\documentclass{article}

\usepackage{amsmath}
\usepackage{amsthm}
\usepackage{amssymb}
\usepackage{amscd}
\usepackage{xspace}
\usepackage{verbatim}

\numberwithin{equation}{section}
\newtheorem{theorem}{Theorem}[section]
\newtheorem{corollary}{Corollary}[section]
\newtheorem{lemma}{Lemma}[section]
\newtheorem{proposition}{Proposition}[section]
\theoremstyle{definition}
\newtheorem{definition}{Definition}[section]

\theoremstyle{remark}
\newtheorem{remark}{Remark}[section]

\newcommand{\ov}{\overline}

\renewcommand{\d}{\delta}

\renewcommand{\O}{\Omega}

\renewcommand{\S}{\Sigma}

\renewcommand{\vec}[1]{\mathbf{#1}}
\newcommand{\R}{\mathbb{R}}

\newcommand{\er}{\eqref}

\DeclareMathOperator{\Div}{div}

\renewcommand{\d}{\delta}

\renewcommand{\O}{\Omega}

\renewcommand{\S}{\Sigma}

\date{}

\begin{document}
\title{
{\Huge Non-relativistic model of the laws of gravity and
electromagnetism,
invariant under the change of inertial and non-inertial coordinate
systems}\\[3mm] This is a preprint of the following work: Arkady
Poliakovsky, "The Laws of Gravity and Electromagnetism: A
Non-relativistic Model Invariant Under the Change of Inertial and
Non-inertial Coordinate Systems", 2024, Springer Cham, reproduced
with permission of Springer Nature Switzerland AG 2024. The final
authenticated version is available online at:
https://doi.org/10.1007/978-3-031-61407-1 }\maketitle
\begin{center}
\textsc{Arkady Poliakovsky \footnote{E-mails:
poliakov@math.bgu.ac.il, arkady.pol@gmail.com}}
\\[3mm]
Department of Mathematics, Ben Gurion University of the Negev,\\
P.O.B. 653, Be'er Sheva 84105, Israel
\\[2mm]
\end{center}
\begin{abstract}
Under the classical non-relativistic consideration of the space-time
we propose the model of the laws of gravity and Electrodynamics,
invariant under the galilean transformations and moreover, under
every change of non-inertial cartesian coordinate system. Being in
the frames of non-relativistic model of the space-time, we adopt
some general ideas of the General Theory of Relativity, like the
assumption of invariance of the most general physical laws in every
inertial and non-inertial coordinate system and equivalence of
factious forces in non-inertial coordinate systems and the force of
gravity. Moreover, in the frames of our model, we obtain that the
laws of Non-relativistic Quantum Mechanics are also invariant under
the change of inertial or non-inertial cartesian coordinate system.
%
%
%
%
%
%
\end{abstract}

\section{Introduction}
\subsection{A new look to the Newtonian Gravity}
Consider the classical space-time where the change of some inertial
coordinate system $(*)$ to another inertial coordinate system $(**)$
is given (up to equivalence) by the Galilean Transformation:
\begin{equation}\label{noninchgravortbstrjgghguittu1int}
\begin{cases}
\vec x'=\vec x+\vec wt,\\
t'=t,
\end{cases}
\end{equation}
and the change of some non-inertial cartesian coordinate system
$(*)$ to another cartesian coordinate system $(**)$ is of the form:
\begin{equation}\label{noninchgravortbstrjgghguittu2int}
\begin{cases}
\vec x'=A(t)\cdot\vec x+\vec z(t),\\
t'=t,
\end{cases}
\end{equation}
where $A(t)\in SO(3)$ is a rotation, i.e. $A(t)\in \R^{3\times 3}$,
$det\, A(t)>0$ and $A(t)\cdot A^T(t)=I$, where $A^T$ is the
transpose of the matrix $A$ and $I$ is the identity matrix.

 Similarly to the General Theory of Relativity, we assume that
the most general laws of Classical Mechanics should be invariant in
every non-inertial cartesian coordinate system, i.e. they preserve
their form under transformations of the form
\er{noninchgravortbstrjgghguittu2int}. Moreover, again as in the
General Theory of Relativity, we assume that the fictitious forces
in non-inertial coordinate systems and the forces of Newtonian
gravitation have the same nature and represented by some field in
somewhat similar to the Electromagnetic field.

 We begin with some simple observation. Assume that we are away from
essential gravitational masses. Then consider two cartesian
coordinate systems $(*)$ and $(**)$, such that the system $(**)$ is
inertial and the change of coordinate system $(*)$ to coordinate
system $(**)$ is given by \er{noninchgravortbstrjgghguittu2int}.
Then the fictitious-gravitational force in the system $(**)$ is
trivial $\vec F'_0=0$. On the other hand, by
\er{noninchgravortbstrjgghguittu2int} the fictitious-gravitational
force in the system $(*)$, acting on the particle with inertial mass
$m$ and velocity $\vec u$, is given by
\begin{equation}\label{noninchgravortbstrjgghguittu2gjgghhjhghjhjggint}
\vec F_0=m\left(-2A^T(t)\cdot\frac{dA}{dt}(t)\cdot\vec
u-A^T(t)\cdot\frac{d^2 A}{dt^2}(t)\cdot\vec
x-A^T(t)\cdot\frac{d^2\vec z}{dt^2}(t)\right).
\end{equation}
Thus if we define a vector field $\vec v:=\vec v(\vec x,t)$ by
\begin{equation}\label{noninchgravortbstrjgghguittu2gjgjhjhhklkint}
\vec v(\vec x,t):=-A^T(t)\cdot\frac{dA}{dt}(t)\cdot\vec
x-A^T(t)\cdot\frac{d\vec z}{dt}(t),
\end{equation}
then, by straightforward calculations we rewrite
\er{noninchgravortbstrjgghguittu2gjgghhjhghjhjggint} as
\begin{equation}\label{noninchgravortbstrjgghguittu2gjgghhjhghjhjgghgghghghtytytint}
\vec F_0=m\left(\frac{\partial\vec v}{\partial
t}+\frac{1}{2}\nabla_{\vec x}\left(|\vec v|^2\right)\right)+m\vec
u\times \left(-curl_{\vec x}\vec v\right)
\end{equation}
(see section \ref{gugyu} for details).

 Similarly, we assume that also in the general case of
gravitational masses there exists a vector field $\vec v:=\vec
v(\vec x,t)$ such that in some inertial or non-inertial cartesian
coordinate system the fictitious-gravitational force is given by
\er{noninchgravortbstrjgghguittu2gjgghhjhghjhjgghgghghghtytytint}.
Then we call the vector field $\vec v$ the vectorial gravitational
potential. We see here the following analogy with Electrodynamics:
denoting
\begin{equation*}
\tilde {\vec E}:=\partial_{t}\vec v+\nabla_{\vec
x}\left(\frac{1}{2}|\vec v|^2\right)\quad\text{and}\quad \tilde
{\vec B}:=-c\, curl_{\vec x}\vec v,
\end{equation*}
we rewrite
\er{noninchgravortbstrjgghguittu2gjgghhjhghjhjgghgghghghtytytint} as
\begin{equation*}
\vec F_0=m\left(\tilde {\vec E}+\frac{1}{c}\vec u\times\tilde {\vec
B}\right),
\end{equation*}
where
\begin{equation*}
curl_{\vec x}\tilde {\vec E}+\frac{1}{c}\frac{\partial}{\partial
t}\tilde {\vec B}=0\quad\text{and}\quad div_{\vec x}\tilde {\vec
B}=0.
\end{equation*}

Next using
\er{noninchgravortbstrjgghguittu2gjgghhjhghjhjgghgghghghtytytint} we
rewrite the Second Law of Newton as
\begin{equation}\label{noninchgravortbstrjgghguittu2gjgghhjhghjhjgghgghghghtytythvfghfgghjggint}
m\frac{d^2\vec x}{dt^2}=m\frac{d\vec u}{dt}=\vec F_0+\vec
F=m\left(\frac{\partial\vec v}{\partial t}(\vec
x,t)+\frac{1}{2}\nabla_{\vec x}\left(\left|\vec
v\right|^2\right)(\vec x,t)\right)+m\vec u\times \left(-curl_{\vec
x}\vec v(\vec x,t)\right)+\vec F,
\end{equation}
where $\vec x:=\vec x(t)$, $\vec u:=\vec u(t)=\frac{d\vec x}{dt}(t)$
and $m$ are the place, the velocity and the inertial mass of some
given particle at the moment of time $t$, $\vec v:=\vec v(\vec x,t)$
is the vectorial gravitational potential and $\vec F$ is the total
non-gravitational force, acting on the given particle.

Once we considered the Second Law of Newton in the form
\er{noninchgravortbstrjgghguittu2gjgghhjhghjhjgghgghghghtytythvfghfgghjggint}
we show that this law is invariant under the change of inertial or
non-inertial cartesian coordinate system, provided that the law of
transformation of the vectorial gravitational potential, under the
change of coordinate system given by
\er{noninchgravortbstrjgghguittu2int}, is:
\begin{equation}\label{MaxVacFull1ninshtrgravortPPNintspd}\vec v'=A(t)\cdot \vec
v+\frac{dA}{dt}(t)\cdot\vec x+\frac{d\vec z}{dt}(t)
\end{equation}
i.e. it is the same as the transformation of a field of velocities.
More precisely we have the following theorem (see section
\ref{gugyu} for the proof):
\begin{theorem}\label{gjghghgghgint}
Consider that the change of some non-inertial cartesian coordinate
system $(*)$ to another cartesian coordinate system $(**)$ is given
by \er{noninchgravortbstrjgghguittu2int}. Next, assume that in the
coordinate system $(**)$ we observe a validity of the Second Law of
Newton in the form:
\begin{equation}\label{MaxVacFull1ninshtrgravortPPNint}
\frac{d\vec u'}{dt'}=-\vec u'\times curl_{\vec x'}\vec
v'+\partial_{t'}\vec v'+\nabla_{\vec x'}\left(\frac{1}{2}|\vec
v'|^2\right)+\frac{1}{m'}\vec F',
\end{equation}
where $\vec x':=\vec x'(t')$, $\vec u':=\vec u'(t')=\frac{d\vec
x'}{dt'}(t')$ and $m'$ are the place, the velocity and the inertial
mass of some given particle at the moment of time $t'$, $\vec
v':=\vec v'(\vec x',t')$ is the vectorial gravitational potential
and $\vec F'$ is a total non-gravitational force, acting on the
given particle in the coordinate system $(**)$. Then in the
coordinate system $(*)$ we have validity of the Second Law of Newton
in the same as \er{MaxVacFull1ninshtrgravortPPNint} form:
\begin{equation}\label{MaxVacFull1ninshtrgravortjhhjPPNjffjfint}
\frac{d\vec u}{dt}=-\vec u\times curl_{\vec x}\vec
v+\partial_{t}\vec v+\nabla_{\vec x}\left(\frac{1}{2}|\vec
v|^2\right)+\frac{1}{m}\vec F,
\end{equation}
provided that
\begin{align}
\label{NoIn5gravortPPN11int}\vec v'=A(t)\cdot \vec
v+\frac{dA}{dt}(t)\cdot\vec x+\frac{d\vec z}{dt}(t)\\
\label{NoIn1gravortPPN11int}\vec F'=A(t)\cdot\vec F,\\
\label{NoIn2gravortPPN11int}m'=m,\\
\label{NoIn3gravortPPN11int}\vec u'=A(t)\cdot \vec
u+\frac{dA}{dt}(t)\cdot\vec x+\frac{d\vec z}{dt}(t).
\end{align}
\end{theorem}
We call a vector field that transforms by
\er{MaxVacFull1ninshtrgravortPPNintspd} under the change of
cartesian coordinate system, by the name speed-like vector field.
Since the vectorial gravitational potential $\vec v$ is a speed-like
vector field, i.e. under the changes of inertial or non-inertial
coordinate system it behaves like a field of velocities of some
continuum, we could introduce a fictitious continuum medium covering
all the space, that we can call Aether, such that $\vec v(\vec x,t)$
is a fictitious velocity of this medium in the point $\vec x$ at the
time $t$. Furthermore, if some particle with the place $\vec r:=\vec
r(t)$, the velocity $\vec u:=\vec u(t)=\frac{d\vec r}{dt}(t)$ and
the inertial mass $m$ moves in the outer gravitational field with
the vectorial gravitational potential $\vec v:=\vec v(\vec x,t)$ in
the absence of non-gravitational forces, then we can associate a
Lagrangian with
\er{noninchgravortbstrjgghguittu2gjgghhjhghjhjgghgghghghtytythvfghfgghjggint}.
Indeed, for this case we define a Lagrangian:
\begin{equation}\label{vhfffngghkjgghfjjSYSPNint}
\mathcal{L}_0\left(\frac{d\vec r}{dt},\vec
r,t\right):=\frac{m}{2}\left|\frac{d\vec r}{dt}-\vec v(\vec
r,t)\right|^2.
\end{equation}
This Lagrangian is invariant under the change of non-inertial
cartesian coordinate systems, given by
\er{noninchgravortbstrjgghguittu2int}. Moreover, we can easily
deduce that a trajectory $\vec r(t):[0,T]\to\mathbb{R}^3$ is a
critical point of the functional
\begin{equation}\label{btfffygtgyggyijhhkkSYSPNint}
I_0=\int_0^T \mathcal{L}_0\left(\frac{d\vec r}{dt}(t),\vec
r(t),t\right)dt
\end{equation}
if and only if it satisfies
\begin{equation}\label{vhfffngghkjgghggtghjgfhjoyuiyuyhiyyukukyihyuSYSPNint}
-m\frac{d^2\vec r}{dt^2}+m\left(\frac{\partial}{\partial t}\vec
v(\vec r,t)+\nabla_{\vec x}\left(\frac{1}{2}\left|\vec v(\vec
r,t)\right|^2\right)-\frac{d\vec r}{dt}\times curl_{\vec x}\vec
v\left(\vec r,t\right)\right)=0,
\end{equation}
consistently with
\er{noninchgravortbstrjgghguittu2gjgghhjhghjhjgghgghghghtytythvfghfgghjggint}
for the case $\vec F=0$. Moreover, we would like to note that if in
some inertial or non-inertial cartesian coordinate system some
material body with the place $\vec r(t)$ and velocity $\vec
u(t)=\frac{d\vec r}{dt}(t)$ moves in the gravitational field, and,
except of the force of gravity all other forces, acting on the body,
are negligible then we can prove that the following equality for
some instant of time $t_0$:
\begin{equation*}
\vec u(t_0):=\frac{d\vec r}{dt}(t_0)=\vec v\left(\vec
r(t_0),t_0\right)
\end{equation*}
implies
\begin{equation*}
\vec u(t):=\frac{d\vec r}{dt}(t)=\vec v\left(\vec r(t),t\right),
\end{equation*}
for every instant of time. I.e. if the velocity of the particle for
some initial instant of time coincides with the local vectorial
gravitational potential, then it will coincide with it at any
instant of time and the trajectory of the motion will be tangent to
the direction of the local vectorial gravitational potential.

Next, in order to fit the Second Law of Newton in the form
\er{noninchgravortbstrjgghguittu2gjgghhjhghjhjgghgghghghtytythvfghfgghjggint}
with the classical Second Law of Newton and the Newtonian Law of
Gravity we consider that in \underline{inertial} coordinate system
$(*)$, at least in the first approximation, we should have
\begin{equation}
\label{MaxVacFull1ninshtrgravortghhghgjkgghklhjgkghghjjkjhjkkggjkhjkhjjhhfhjhklkhkhjjklzzzyyyhjggjhgghhjhNWNWBWHWPPN222kgghjghjghjint}
\begin{cases}
curl_{\vec x}\vec v= 0,\\
\frac{\partial\vec v}{\partial t}+\frac{1}{2}\nabla_{\vec
x}\left(|\vec v|^2\right)= -\nabla_{\vec x}\Phi,
\end{cases}
\end{equation}
where $\Phi$ is a scalar Newtonian gravitational potential which
satisfies
\begin{equation}
\label{MaxVacFull1ninshtrgravortghhghgjkgghklhjgkghghjjkjhjkkggjkhjkhjjhhfhjhklkhkhjjklzzzyyyhjggjhgghhjhNWNWNWBWHWPPN222int}
\Delta_{\vec x}\Phi=4\pi GM,
\end{equation}
where $M$ is the gravitational mass density and $G$ is the
gravitational constant. Thus, since we require $curl_{\vec x}\vec v=
0$,
\er{MaxVacFull1ninshtrgravortghhghgjkgghklhjgkghghjjkjhjkkggjkhjkhjjhhfhjhklkhkhjjklzzzyyyhjggjhgghhjhNWNWBWHWPPN222kgghjghjghjint}
is equivalent to:
\begin{equation}
\label{MaxVacFull1ninshtrgravortghhghgjkgghklhjgkghghjjkjhjkkggjkhjkhjjhhfhjhklkhkhjjklzzzyyyhjggjhgghhjhNWNWBWHWPPN222int}
\begin{cases}
curl_{\vec x}\vec v= 0,\\
\frac{\partial\vec v}{\partial t}+d_\vec x\vec v\cdot\vec v=
-\nabla_{\vec x}\Phi,
\end{cases}
\end{equation}
where $d_{\vec x}\vec v$ is the Jacobian matrix of the vector field
$\vec v$. Clearly the law
\er{MaxVacFull1ninshtrgravortghhghgjkgghklhjgkghghjjkjhjkkggjkhjkhjjhhfhjhklkhkhjjklzzzyyyhjggjhgghhjhNWNWBWHWPPN222int}
is invariant under the change of inertial coordinate system, given
by \er{noninchgravortbstrjgghguittu1int}. Note also that, since in
the system $(*)$ we have $curl_{\vec x}\vec v=0$, we can write
\er{MaxVacFull1ninshtrgravortghhghgjkgghklhjgkghghjjkjhjkkggjkhjkhjjhhfhjhklkhkhjjklzzzyyyhjggjhgghhjhNWNWBWHWPPN222kgghjghjghjint}
as a Hamilton-Jacobi type equation:
\begin{equation}
\label{MaxVacFull1ninshtrgravortghhghgjkgghklhjgkghghjjkjhjkkggjkhjkhjjhhfhjhklkhkhjjklzzzyyyhjggjhgghhjhNWNWNWNWNWBWHWPPN222int}
\begin{cases}
\vec v=\nabla_{\vec x}Z,\\
\frac{\partial Z}{\partial t}+\frac{1}{2}\left|\nabla_{\vec
x}Z\right|^2=-\Phi,
\end{cases}
\end{equation}
where $Z$ is some scalar field. Next we introduce a law of gravity
which is invariant in every non-inertial cartesian coordinate system
and is equivalent to
\er{MaxVacFull1ninshtrgravortghhghgjkgghklhjgkghghjjkjhjkkggjkhjkhjjhhfhjhklkhkhjjklzzzyyyhjggjhgghhjhNWNWBWHWPPN222int}
in every inertial coordinate system. This law has the form:
\begin{equation}
\label{MaxVacFull1ninshtrgravortghhghgjkgghklhjgkghghjjkjhjkkggjkhjkhjjhhfhjhklkhkhjjklzzzyyyNWBWHWPPN222int}
\begin{cases}
curl_{\vec x}\left(curl_{\vec x}\vec v\right)= 0,\\
\frac{\partial}{\partial t}\left(div_{\vec x}\vec v\right)+div_{\vec
x}\left\{\left(div_{\vec x}\vec v\right)\vec
v\right\}+\frac{1}{4}\left|d_{\vec x}\vec v+\{d_{\vec x}\vec
v\}^T\right|^2-\left(div_{\vec x}\vec v\right)^2= -4\pi GM,
\end{cases}
\end{equation}
%
%
%
%
%
%
together with some reasonable frame-independent assumptions for the
asymptotic behavior of the vectorial gravitational potential $\vec
v$ as $|\vec x|\to+\infty$ (see section \ref{gugyu} for the details,
in particular, see Proposition \ref{jghjgyfyffffhfgljjkjl}).

Next one can wonder: what should be possible values of the vectorial
gravitational potential $\vec v$ in the proximity of the Earth or
another massive body? We attempt to answer this question in remark
\ref{gygygygyggjh}. We obtain there that, if we consider a
\underline{non-rotating} cartesian coordinate system which center
coincides with the center of the Earth, then in this system we
should have either
\begin{equation}
\label{MaxVacFull1ninshtrgravortghhghgjkgghklhjgkghghjjkjhjkkggjkhjkhjjhhfhjhklkhkhjjklzzzyyyhjggjhgghhjhNWNWBWHWPPN222kkkgtytghjjhpoppkkkhhhkhjhjint}
\vec v(\vec x)= \frac{\sqrt{-2\Phi_1(|\vec x|)}}{|\vec x|}\vec x,
\end{equation}
or
\begin{equation}
\label{MaxVacFull1ninshtrgravortghhghgjkgghklhjgkghghjjkjhjkkggjkhjkhjjhhfhjhklkhkhjjklzzzyyyhjggjhgghhjhNWNWBWHWPPN222kkkgtytghjjhpoppkkkhhhkhjhjiuuint}
\vec v(\vec x)= -\frac{\sqrt{-2\Phi_1(|\vec x|)}}{|\vec x|}\vec x,
\end{equation}
where $\Phi_1$ is the usual Newtonian potential of the Earth, that
satisfies $\Phi_1(r)=-\frac{Gm_0}{r}$ outside of the Earth. In
particular, on the Earth surface we have:
\begin{equation}
\label{MaxVacFull1ninshtrgravortghhghgjkgghklhjgkghghjjkjhjkkggjkhjkhjjhhfhjhklkhkhjjklzzzyyyhjggjhgghhjhNWNWBWHWPPN222kkkgtytghjjhpoppkkkhhhkhjhjiuuokokokint}
|\vec v|=\sqrt{\frac{2Gm_0}{r_0}},
\end{equation}
where $r_0$ is the Earth radius and $m_0$ is the Earth mass, i.e.
the absolute value of the vectorial gravitational potential on the
Earth surface approximately equals to the escape velocity and its
direction is normal to the Earth, either downward or upward.

\subsubsection{Variational principle for the solution of the
Cauchy-problem for Hamilton-Jacobi type equation
\er{MaxVacFull1ninshtrgravortghhghgjkgghklhjgkghghjjkjhjkkggjkhjkhjjhhfhjhklkhkhjjklzzzyyyhjggjhgghhjhNWNWNWNWNWBWHWPPN222int}.}
Consider the Hamilton-Jacobi type equation
\er{MaxVacFull1ninshtrgravortghhghgjkgghklhjgkghghjjkjhjkkggjkhjkhjjhhfhjhklkhkhjjklzzzyyyhjggjhgghhjhNWNWNWNWNWBWHWPPN222int}
for the vectorial gravitational potential in the non-rotating
coordinate system:
\begin{equation}
\label{MaxVacFull1ninshtrgravortghhghgjkgghklhjgkghghjjkjhjkkggjkhjkhjjhhfhjhklkhkhjjklzzzyyyhjggjhgghhjhNWNWNWNWNWBWHWPPN222RRint}
\begin{cases}
\vec v=\nabla_{\vec x}Z,\\
\frac{\partial Z}{\partial t}+\frac{1}{2}\left|\nabla_{\vec
x}Z\right|^2=-\Phi,
\end{cases}
\end{equation}
where $Z:=Z(\vec x,t)$ is some scalar field and $\Phi=\Phi(\vec
x,t)$ is the scalar Newtonian gravitational potential which
satisfies
\er{MaxVacFull1ninshtrgravortghhghgjkgghklhjgkghghjjkjhjkkggjkhjkhjjhhfhjhklkhkhjjklzzzyyyhjggjhgghhjhNWNWNWBWHWPPN222int}:
\begin{equation}
\label{MaxVacFull1ninshtrgravortghhghgjkgghklhjgkghghjjkjhjkkggjkhjkhjjhhfhjhklkhkhjjklzzzyyyhjggjhgghhjhNWNWNWBWHWPPN222RRint}
\Delta_{\vec x}\Phi=4\pi GM,
\end{equation}
where $M$ is the gravitational mass density and $G$ is the
gravitational constant. Next consider the Cauchy problem for
\er{MaxVacFull1ninshtrgravortghhghgjkgghklhjgkghghjjkjhjkkggjkhjkhjjhhfhjhklkhkhjjklzzzyyyhjggjhgghhjhNWNWNWNWNWBWHWPPN222RRint}:
\begin{equation}
\label{MaxVacFull1ninshtrgravortghhghgjkgghklhjgkghghjjkjhjkkggjkhjkhjjhhfhjhklkhkhjjklzzzyyyhjggjhgghhjhNWNWNWNWNWBWHWPPN222RRkmmjljlmint}
\begin{cases}
\frac{\partial Z}{\partial t}(\vec
x,t)+\frac{1}{2}\left|\nabla_{\vec x}Z(\vec x,t)\right|^2=-\Phi(\vec
x,t),\\
Z(\vec x,0)=\varphi(\vec x),
\end{cases}
\end{equation}
where $\Phi(\vec x,t)$ and $\varphi(\vec x)$ are prescribed smooth
scalar functions.
Furthermore, for
every instant of time $t\geq 0$ consider the followed variational
functional defined on trajectories $\vec r(s):[0,t]\to\R^3$:
\begin{equation}\label{gjgjfhhhhhhhhhhhhhhhhhjjjjkint}
J_{gr}\left(\vec r\right)=\varphi\left(\vec
r(0)\right)+\int_0^t\left(\frac{1}{2}\left|\frac{d\vec
r}{ds}(s)\right|^2-\Phi\left(\vec r(s),s\right)\right)ds.
\end{equation}
Then, inserting the smooth solution $Z$ of
\er{MaxVacFull1ninshtrgravortghhghgjkgghklhjgkghghjjkjhjkkggjkhjkhjjhhfhjhklkhkhjjklzzzyyyhjggjhgghhjhNWNWNWNWNWBWHWPPN222RRkmmjljlmint}
into \er{gjgjfhhhhhhhhhhhhhhhhhjjjjkint} and using the Chain Rule,
in subsection \ref{GGGGGZZZZffffhhh} we deduce that we can rewrite
the definition of the functional $J_{gr}$ in
\er{gjgjfhhhhhhhhhhhhhhhhhjjjjkint} as:
\begin{multline}\label{gjgjfhhhhhhhhhhhhhhhhhjjjjklmjkkhkjjjjmjljint}
J_{gr}\left(\vec r\right)=Z\left(\vec
r(t),t\right)+\frac{1}{2}\int_0^t\left|\frac{d\vec
r}{ds}(s)-\nabla_{\vec {\vec r}}Z\left(\vec
r(s),s\right)\right|^2ds=Z\left(\vec
r(t),t\right)+\frac{1}{2}\int_0^t\left|\frac{d\vec r}{ds}(s)-\vec
v\left(\vec r(s),s\right)\right|^2ds,
\end{multline}
where consistently with
\er{MaxVacFull1ninshtrgravortghhghgjkgghklhjgkghghjjkjhjkkggjkhjkhjjhhfhjhklkhkhjjklzzzyyyhjggjhgghhjhNWNWNWNWNWBWHWPPN222RRint}
we consider $\vec v(\vec x,t):=\nabla_{\vec x}Z(\vec x,t)$.
Therefore, if for every point $\vec x\in\R^3$ and every instant of
time $t\geq 0$ we consider the function $g:=g(\vec
x,t):\R^3:[0,+\infty)\to\R$ defined as a minimum of the following
variational problem:
\begin{multline}\label{gjgjfhhhhhhhhhhhhhhhhhjjjjkhhhjhjjhhint}
g(\vec x,t):=\min{\left\{J_{gr}\left(\vec r\right)\;\mid\;\;\vec
r(s):[0,t]\to\R^3,\;r(t)=\vec x\right\}}=\\
\min{\left\{\varphi\left(\vec
r(0)\right)+\int_0^t\left(\frac{1}{2}\left|\frac{d\vec
r}{ds}(s)\right|^2-\Phi\left(\vec
r(s),s\right)\right)ds\;\mid\;\;\vec r(s):[0,t]\to\R^3,\;\vec
r(t)=\vec x\right\}},
\end{multline}
then by \er{gjgjfhhhhhhhhhhhhhhhhhjjjjklmjkkhkjjjjmjljint} we deduce
that
\begin{equation}\label{gjgjfhhhhhhhhhhhhhhhhhjjjjkhhhjhjjhhjkjhhint}
g(\vec x,t)=Z(\vec x,t),
\end{equation}
and this minimum is achieved on the unique trajectory $\vec
r(s):[0,t]\to\R^3$ that satisfies the following initial value
problem for an ordinary differential equation:
\begin{equation}
\label{MaxVacFull1ninshtrgravortghhghgjkgghklhjgkghghjjkjhjkkggjkhjkhjjhhfhjhklkhkhjjklzzzyyyhjggjhgghhjhNWNWNWNWNWBWHWPPN222RRkkljljkhhjhint}
\begin{cases}
\frac{d\vec r}{ds}(s)=\nabla_{\vec {\vec r}}Z\left(\vec
r(s),s\right)=\vec v\left(\vec r(s),s\right)\quad\quad\forall
s\in[0,t],\\
\vec r(t)=\vec x.
\end{cases}
\end{equation}
So by inserting \er{gjgjfhhhhhhhhhhhhhhhhhjjjjkhhhjhjjhhjkjhhint}
into \er{gjgjfhhhhhhhhhhhhhhhhhjjjjkhhhjhjjhhint} we obtain the
\underline{explicit} formula for the solution $Z$ of the Chuchy
problem
\er{MaxVacFull1ninshtrgravortghhghgjkgghklhjgkghghjjkjhjkkggjkhjkhjjhhfhjhklkhkhjjklzzzyyyhjggjhgghhjhNWNWNWNWNWBWHWPPN222RRkmmjljlmint}:
\begin{multline}\label{gjgjfhhhhhhhhhhhhhhhhhjjjjkhhhjhjjhhgggint}
Z(\vec x,t)=\min{\left\{\varphi\left(\vec
r(0)\right)+\int_0^t\left(\frac{1}{2}\left|\frac{d\vec
r}{ds}(s)\right|^2-\Phi\left(\vec
r(s),s\right)\right)ds\;\mid\;\;\vec r(s):[0,t]\to\R^3,\;\vec
r(t)=\vec x\right\}}.
\end{multline}
Furthermore, there exists a unique trajectory $\vec
r(s):[0,t]\to\R^3$ that minimizes the right-hand-side of
\er{gjgjfhhhhhhhhhhhhhhhhhjjjjkhhhjhjjhhgggint} and moreover it
satisfies
\er{MaxVacFull1ninshtrgravortghhghgjkgghklhjgkghghjjkjhjkkggjkhjkhjjhhfhjhklkhkhjjklzzzyyyhjggjhgghhjhNWNWNWNWNWBWHWPPN222RRkkljljkhhjhint}.
In particular for this trajectory we have
\begin{equation}\label{MaxVacFull1ninshtrgravortghhghgjkgghklhjgkghghjjkjhjkkggjkhjkhjjhhfhjhklkhkhjjklzzzyyyhjggjhgghhjhNWNWNWNWNWBWHWPPN222RRkkljljkhhjhjhhhjhkjjgint}
\vec v\left(\vec x,t\right)=\frac{d\vec r}{ds}(t).
\end{equation}
On the other hand the minimizer of the right-hand-side of
\er{gjgjfhhhhhhhhhhhhhhhhhjjjjkhhhjhjjhhgggint} clearly satisfies
the following Euler-Lagrange equation:
\begin{equation}
\label{MaxVacFull1ninshtrgravortghhghgjkgghklhjgkghghjjkjhjkkggjkhjkhjjhhfhjhklkhkhjjklzzzyyyhjggjhgghhjhNWNWNWNWNWBWHWPPN222RRkkljljkhhjhjhhhjhjkhint}
\begin{cases}
\frac{d^2\vec r}{ds^2}(s)=-\nabla_{\vec {\vec r}}\Phi\left(\vec
r(s),s\right)\quad\quad\forall
s\in[0,t],\\
\frac{d\vec r}{ds}(0)=\nabla_{\vec r}\varphi\left(\vec r(0)\right),
\\
\vec r(t)=\vec x.
\end{cases}
\end{equation}
Equation
\er{MaxVacFull1ninshtrgravortghhghgjkgghklhjgkghghjjkjhjkkggjkhjkhjjhhfhjhklkhkhjjklzzzyyyhjggjhgghhjhNWNWNWNWNWBWHWPPN222RRkkljljkhhjhjhhhjhjkhint}
sometimes has an advantage with respect to
\er{MaxVacFull1ninshtrgravortghhghgjkgghklhjgkghghjjkjhjkkggjkhjkhjjhhfhjhklkhkhjjklzzzyyyhjggjhgghhjhNWNWNWNWNWBWHWPPN222RRkkljljkhhjhint}
since the a priori unknown function $Z$ dose not appear in
\er{MaxVacFull1ninshtrgravortghhghgjkgghklhjgkghghjjkjhjkkggjkhjkhjjhhfhjhklkhkhjjklzzzyyyhjggjhgghhjhNWNWNWNWNWBWHWPPN222RRkkljljkhhjhjhhhjhjkhint}.

So, in order to find $\vec v(\vec x,t)$ we need first to solve
\er{MaxVacFull1ninshtrgravortghhghgjkgghklhjgkghghjjkjhjkkggjkhjkhjjhhfhjhklkhkhjjklzzzyyyhjggjhgghhjhNWNWNWNWNWBWHWPPN222RRkkljljkhhjhjhhhjhjkhint}
and then use
\er{MaxVacFull1ninshtrgravortghhghgjkgghklhjgkghghjjkjhjkkggjkhjkhjjhhfhjhklkhkhjjklzzzyyyhjggjhgghhjhNWNWNWNWNWBWHWPPN222RRkkljljkhhjhjhhhjhkjjgint}
with the solution of
\er{MaxVacFull1ninshtrgravortghhghgjkgghklhjgkghghjjkjhjkkggjkhjkhjjhhfhjhklkhkhjjklzzzyyyhjggjhgghhjhNWNWNWNWNWBWHWPPN222RRkkljljkhhjhjhhhjhjkhint}
that we just found.

\subsection{Non-relativistic model of Electrodynamics}
Similarly to the General Theory of Relativity we assume that the
electromagnetic field is influenced by the gravitational field. In
Section \ref{MaxRevsPPN} of this paper we propose the simple and
natural quantitative relations of Electrodynamics, substituting
(with minor changes) the classical Maxwell equations in the case of
an arbitrarily vectorial gravitational potential, and invariant
under Galilean Transformations. For this propose we appeal to the
Maxwell equations in a medium. It is well known that
the classical Maxwell equations in a medium have the following form
in the Gaussian unit system:
\begin{equation}
\label{MaxVackkkint}
\begin{cases}
curl_{\vec x} \vec H= \frac{4\pi}{c}\vec j+\frac{1}{c}\frac{\partial \vec D}{\partial t},\\
div_{\vec x} \vec D= 4\pi\rho,\\
curl_{\vec x} \vec E+\frac{1}{c}\frac{\partial \vec B}{\partial t}= 0,\\
div_{\vec x} \vec B= 0.
\end{cases}
\end{equation}
Here $\vec x\in\R^3$ and $t>0$ are the place and the time, $\vec E$
is the electric field, $\vec B$ is the magnetic field, $\vec D$ is
the electric displacement field, $\vec H$ is the $\vec H$-magnetic
field, $\rho$ is the charge density, $\vec j$ is the current density
and $c$ is the universal constant,
called
speed of light.
It is assumed in the Classical Electrodynamics that for the vacuum
we always have $\vec D= \vec E$ and $\vec H=\vec B$. We assume here
that the Maxwell equations in the vacuum have the usual form of
\er{MaxVackkkint} in every inertial coordinate system, as in any
other medium, however, we assume that, given some inertial
coordinate system, the relations $\vec D= \vec E$ and $\vec H=\vec
B$ in the vacuum are valid only for the parts of the space, where
the vectorial gravitational potential is negligible.

So we assume that, given some inertial coordinate system, if in some
point and at some instant the vectorial gravitational potential
vanishes, then in this point and at this time we have $\vec D= \vec
E$ and $\vec H=\vec B$. In order to obtain the relations $\vec D\sim
\vec E$ and $\vec H\sim \vec B$ in the general case we assume that
the equations \er{MaxVackkkint} and the Lorentz force
\begin{equation}
\label{LorenzChlljklljklkkzzzint} \vec F=\sigma \vec
E+\frac{\sigma}{c}\,\vec u\times \vec B
\end{equation}
(where $\sigma$ is the charge of the test particle and $\vec u$ is
its velocity) are invariant under the Galilean transformations,
given by \er{noninchgravortbstrjgghguittu1int}. Then the analysis of
our assumptions, presented in section \ref{MaxRevsPPN}, implies that
the full system of Electrodynamics in the case of an arbitrarily
vectorial gravitational potential $\vec v:=\vec v(\vec x,t)$ has the
following form:
\begin{equation}\label{MaxVacFull1bjkgjhjhgjaaaint}
\begin{cases}
curl_{\vec x} \vec H= \frac{4\pi}{c}\vec j+\frac{1}{c}\frac{\partial
\vec D}{\partial t},\\
div_{\vec x} \vec D= 4\pi\rho,\\
curl_{\vec x} \vec E+\frac{1}{c}\frac{\partial \vec B}{\partial t}=0,\\
div_{\vec x} \vec B=0,\\
\vec E=\vec D-\frac{1}{c}\,\vec v\times \vec B,\\
\vec H=\vec B+\frac{1}{c}\,\vec v\times \vec D.
\end{cases}
\end{equation}
It can be easily checked that system
\er{MaxVacFull1bjkgjhjhgjaaaint} and the expression of the Lorentz
force in \er{LorenzChlljklljklkkzzzint}
are invariant under the Galilean transformations
\er{noninchgravortbstrjgghguittu1int}, provided that
\begin{equation}\label{EBDHTrans2kkkint}
\begin{cases}
\vec D'=\vec D,\\
\vec B'=\vec B,\\
\vec E'=\vec E-\frac{1}{c}\,\vec w\times \vec B,\\
\vec H'=\vec H+\frac{1}{c}\,\vec w\times \vec D\\
\vec v'=\vec v+\vec w.
\end{cases}
\end{equation}
In section \ref{INNONredPPN} we prove that the laws of
Electrodynamics in the form \er{MaxVacFull1bjkgjhjhgjaaaint} and the
law of the Lorentz force \er{LorenzChlljklljklkkzzzint}, preserve
their form also in non-inertial cartesian coordinate systems. More
precisely the following theorem is valid:
\begin{theorem}\label{gjghghgghgintint}
Consider that the change of some non-inertial cartesian coordinate
system $(*)$ to another cartesian coordinate system $(**)$ is given
by \er{noninchgravortbstrjgghguittu2int}. Next, assume that in the
coordinate system $(**)$ we observe a validity of Maxwell Equations
for the vacuum in the form:
\begin{equation}\label{MaxVacFull1ninshtrredPPNint}
\begin{cases}
curl_{\vec x'} \vec H'= \frac{4\pi}{c}\vec
j'+\frac{1}{c}\frac{\partial
\vec D'}{\partial t'},\\
div_{\vec x'} \vec D'= 4\pi\rho',\\
curl_{\vec x'} \vec E'+\frac{1}{c}\frac{\partial \vec B'}{\partial t'}= 0,\\
div_{\vec x'} \vec B'= 0,\\
\vec E'=\vec D'-\frac{1}{c}\,\vec v'\times \vec B',\\
\vec H'=\vec B'+\frac{1}{c}\,\vec v'\times \vec D'.
\end{cases}
\end{equation}
Moreover, we assume that in coordinate system $(**)$ we observe a
validity of the expression for the Lorentz force in the form:
\begin{equation}\label{LorenzChredPPNint}
\vec F'=\sigma' \vec E'+\frac{\sigma'}{c}\,\vec u'\times \vec B'.
\end{equation}
Then in the coordinate system $(*)$ we have the validity of Maxwell
Equations for the vacuum in the same as
\er{MaxVacFull1ninshtrredPPNint} form:
\begin{equation}\label{MaxVacFull1ninshtrhjkkredPPNint}
\begin{cases}
curl_{\vec x} \vec H= \frac{4\pi}{c}\vec j+\frac{1}{c}\frac{\partial
\vec D}{\partial t},\\
div_{\vec x} \vec D= 4\pi\rho,\\
curl_{\vec x} \vec E+\frac{1}{c}\frac{\partial \vec B}{\partial t}= 0,\\
div_{\vec x} \vec B= 0,\\
\vec E=\vec D-\frac{1}{c}\,\vec v\times \vec B,\\
\vec H=\vec B+\frac{1}{c}\,\vec v\times \vec D,
\end{cases}
\end{equation}
and we have the validity of the expression for the Lorentz force in
the same as \er{LorenzChredPPNint} form:
\begin{equation}\label{LorenzChlljklljkint}
\vec F=\sigma \vec E+\frac{\sigma}{c}\,\vec u\times \vec B,
\end{equation}
provided that
\begin{equation}\label{yuythfgfyftydtydtydtyddyyyhhddhhhredPPN111hgghjgintint}
\begin{cases}
\vec F'=A(t)\cdot\vec F,\\
\sigma'=\sigma,\\
\vec u'=A(t)\cdot \vec u+\frac{dA}{dt}(t)\cdot\vec x+\frac{d\vec z}{dt}(t),\\
\rho'=\rho,\\
\vec v'=A(t)\cdot \vec v+\frac{dA}{dt}(t)\cdot\vec x+\frac{d\vec z}{dt}(t),\\
\vec j'=A(t)\cdot \vec j+\rho\, \frac{dA}{dt}(t)\cdot\vec
x+\rho\,\frac{d\vec z}{dt}(t)
\end{cases}
\end{equation}
and
\begin{equation}\label{yuythfgfyftydtydtydtyddyyyhhddhhhredPPN111hgghjgint}
\begin{cases}
\vec D'=A(t)\cdot \vec D,\\
\vec B'=A(t)\cdot\vec B,\\
\vec E'=A(t)\cdot\vec E-\frac{1}{c}\,\left(\frac{dA}{dt}(t)\cdot\vec
x+\frac{d\vec z}{dt}(t)\right)\times \left(A(t)\cdot\vec B\right),\\
\vec H'=A(t)\cdot\vec H+\frac{1}{c}\,\left(\frac{dA}{dt}(t)\cdot\vec
x+\frac{d\vec z}{dt}(t)\right)\times \left(A(t)\cdot\vec D\right).
\end{cases}
\end{equation}
\end{theorem}

Next we would like to note that, since as already mentioned before,
the direction of the local vectorial gravitational potential is
normal to the Earth surface, in the frames of our model, we provide
a non-relativistic explanation of the classical Michelson-Morley
experiment. Indeed in this experiment the axes of the apparatus are
tangent to the Earth surface and thus the null result cannot be
affected by the vectorial gravitational potential. Since, the value
of the local vectorial gravitational potential equals to the escape
velocity, if we consider the vertical Michelson-Morley experiment,
where one of the axes of the apparatus is normal to the Earth
surface, then in the frames of our model the expected result should
be analogous to the positive result of Aether drift with the speed
equal to the escape velocity. However, regarding the vertical
Michelson-Morley experiment i.e. the modification of
Michelson-Morley experiment, where at least one of the axes of the
apparatus is not tangent to the Earth surface, we found only very
scarce and contradictory information.

Next, as in the classical electrodynamics, by the third and the
fourth equations in \er{MaxVacFull1bjkgjhjhgjaaaint} we can find a
scalar field $\Psi:=\Psi(\vec x,t)$ and a vector field $\vec A:=\vec
A(\vec x,t)$ such that
\begin{equation}\label{MaxVacFull1bjkgjhjhgjgjgkjfhjfdghghligioiuittrPPNint}
\begin{cases}
\vec B\equiv curl_{\vec x} \vec A,\\
\vec E\equiv-\nabla_{\vec x}\Psi-\frac{1}{c}\frac{\partial\vec
A}{\partial t}.
\end{cases}
\end{equation}
We call $\Psi$ and $\vec A$ the scalar and the vectorial
electromagnetic potentials. Then by
\er{MaxVacFull1bjkgjhjhgjgjgkjfhjfdghghligioiuittrPPNint} and
\er{MaxVacFull1bjkgjhjhgjaaaint} we also have
\begin{equation}\label{vhfffngghPPN333yuyuint}
\begin{cases}
\vec D=-\nabla_{\vec x}\Psi-\frac{1}{c}\frac{\partial\vec
A}{\partial t}+\frac{1}{c}\vec
v\times curl_{\vec x}\vec A\\
\vec H\equiv curl_{\vec x} \vec A+\frac{1}{c}\,\vec
v\times\left(-\nabla_{\vec x}\Psi-\frac{1}{c}\frac{\partial\vec
A}{\partial t}+\frac{1}{c}\vec v\times curl_{\vec x}\vec A\right).
\end{cases}
\end{equation}
We also define the proper scalar electromagnetic potential
$\Psi_0:=\Psi_0(\vec x,t)$ by
\begin{equation}\label{vhfffngghhjghhgPPNghghghutghffugghjhjkjjklint}
\Psi_0:=\Psi-\frac{1}{c}\vec A\cdot\vec v.
\end{equation}
The name "proper scalar potential" is clarified below. The
electromagnetic potentials are not uniquely defined and thus we need
to choose a calibration. It is clear that if
$(\tilde\Psi,\tilde\Psi_0,\tilde{\vec A})$ is another choice of
electromagnetic potentials with a different calibration then there
exists a scalar field $w:=w(\vec x,t)$ such that we have
\begin{equation}\label{MaxVacFull1bjkgjhjhgjgjgkjfhjfdghghligioiuittrPPNhjkjhkjgghhjjhjint}
\begin{cases}
\tilde\Psi=\Psi+\frac{1}{c}\frac{\partial w}{\partial t}\\
\tilde{\vec A}=\vec A-\nabla_{\vec x}w\\
\tilde\Psi_0=\Psi_0+\frac{1}{c}\left(\frac{\partial w}{\partial
t}+\vec v\cdot\nabla_{\vec x}w\right).
\end{cases}
\end{equation}
For definiteness we can take $\vec A$ to satisfy
\begin{equation}\label{MaxVacFull1bjkgjhjhgjgjgkjfhjfdghghligioiuittrPPN22int}
div_{\vec x}\vec A\equiv 0.
\end{equation}
In section \ref{ghfvdgfdjfg} we show that, consistently with
\er{yuythfgfyftydtydtydtyddyyyhhddhhhredPPN111hgghjgint}, under the
change of non-inertial cartesian coordinate system, given by
\er{noninchgravortbstrjgghguittu2int}, the electromagnetic
potentials transform as:
\begin{equation}\label{vhfffngghhjghhgPPNghghghutghfflklhjkjhjhjjgjkghhjint}
\begin{cases}
\Psi'=
\Psi+\frac{1}{c}\left(\frac{dA}{dt}(t)\cdot\vec x+\frac{d\vec
z}{dt}(t)\right)\cdot\left(A(t)\cdot\vec A\right)
\\
\vec A'=A(t)\cdot \vec A\\
\Psi'_0=\Psi_0.
\end{cases}
\end{equation}
The last equation in
\er{vhfffngghhjghhgPPNghghghutghfflklhjkjhjhjjgjkghhjint} clarifies
the name "proper scalar potential". The equalities
\er{vhfffngghhjghhgPPNghghghutghfflklhjkjhjhjjgjkghhjint} are
derived primarily under the choice of the calibration given by
\er{MaxVacFull1bjkgjhjhgjgjgkjfhjfdghghligioiuittrPPN22int}.
However, as can be easily seen by
\er{MaxVacFull1bjkgjhjhgjgjgkjfhjfdghghligioiuittrPPNhjkjhkjgghhjjhjint},
all the equalities in
\er{vhfffngghhjghhgPPNghghghutghfflklhjkjhjhjjgjkghhjint} still
remain to hold, under any other choice of calibration scalar
function $w$, provided that we have $w'=w$ under the transformation
\er{noninchgravortbstrjgghguittu2int}. In particular, under the
Galilean transformations \er{noninchgravortbstrjgghguittu1int} the
electromagnetic potentials transform as:
\begin{equation}\label{vhfffngghhjghhgPPNghghghutghfflklhjkjhjhjjgjkghhjhhhjhgjguint}
\begin{cases}
\Psi'=
\Psi+\frac{1}{c}\vec w\cdot\vec A
\\
\vec A'=\vec A\\
\Psi'_0=\Psi_0.
\end{cases}
\end{equation}

Next we can associate a Lagrangian density related to
electromagnetic field. Given known the charge distribution
$\rho:=\rho(\vec x,t)$, the current distribution $\vec j:=\vec
j(\vec x,t)$ and the vectorial gravitational potential $\vec v:=\vec
v(\vec x,t)$, consider a Lagrangian density $L_1$ defined by
\begin{multline}\label{vhfffngghkjgghPPNint}
L_1\left(\vec A,\Psi,\vec
x,t\right):=\frac{1}{8\pi}\left|-\nabla_{\vec
x}\Psi-\frac{1}{c}\frac{\partial\vec A}{\partial t}+\frac{1}{c}\vec
v\times curl_{\vec x}\vec A\right|^2-\frac{1}{8\pi}\left|curl_{\vec
x}\vec A\right|^2-\left(\rho\Psi-\frac{1}{c}\vec A\cdot\vec
j\right).
\end{multline}
Using \er{vhfffngghhjghhgPPNghghghutghfflklhjkjhjhjjgjkghhjint} we
can deduce that Lagrangian $L_1$ is invariant, under the change of
inertial or non-inertial cartesian coordinate system, given by
\er{noninchgravortbstrjgghguittu2int}. Moreover, if, consistently
with \er{MaxVacFull1bjkgjhjhgjgjgkjfhjfdghghligioiuittrPPNint},
\er{vhfffngghPPN333yuyuint} and
\er{vhfffngghhjghhgPPNghghghutghffugghjhjkjjklint}, we denote
\begin{equation}\label{guigjgjffghPPNint}
\begin{cases}
\vec D=-\nabla_{\vec x}\Psi-\frac{1}{c}\frac{\partial\vec
A}{\partial t}+\frac{1}{c}\vec
v\times curl_{\vec x}\vec A\\
\vec B=curl_{\vec x}\vec A
\\
\vec E=-\nabla_{\vec x}\Psi-\frac{1}{c}\frac{\partial\vec A}{\partial t}
\\
\vec H=curl_{\vec x}\vec A+\frac{1}{c}\vec v\times\left(\nabla_{\vec
x}\Psi-\frac{1}{c}\frac{\partial\vec A}{\partial t}+\frac{1}{c}\vec
v\times curl_{\vec x}\vec A\right)
\\
\Psi_0:=\Psi-\frac{1}{c}\vec A\cdot\vec v,
\end{cases}
\end{equation}
then:
\begin{multline}
\label{vhfffngghkjgghPPNggjgjjkgjint} L_1\left(\vec A,\Psi,\vec
x,t\right)=\frac{1}{8\pi}\left|\vec
D\right|^2-\frac{1}{8\pi}\left|\vec
B\right|^2-\left(\rho\Psi-\frac{1}{c}\vec A\cdot\vec
j\right)=\frac{1}{8\pi}\left|\vec
D\right|^2-\frac{1}{8\pi}\left|\vec
B\right|^2-\rho\Psi_0+\frac{1}{c}\vec A\cdot(\vec j-\rho\vec v).
\end{multline}
Then in section \ref{bhjghjfghfg} we obtain that a configuration
$(\Psi,\vec A)$ is a critical point of the functional
\begin{equation}\label{btfffygtgyggyPPNint}
J_0=\int_0^T\int_{\mathbb{R}^3}L_1\left(\vec A(\vec x,t),\Psi(\vec
x,t),\vec x,t\right)d\vec x dt,
\end{equation}
if and only if we have
\begin{equation}\label{guigjgjffghguygjyfPPNint}
\begin{cases}
curl_{\vec x}\vec H=\frac{4\pi}{c}\vec j+\frac{\partial\vec
D}{\partial
t}\\
div_{\vec x}\vec D=4\pi\rho\\
curl_{\vec x}\vec E+\frac{1}{c}\frac{\partial\vec B}{\partial t}=0\\
div_{\vec x}\vec B=0\\
\vec E=\vec D-\frac{1}{c}\vec v\times\vec B\\
\vec H=\vec B+\frac{1}{c}\vec v\times\vec D,
\end{cases}
\end{equation}
where $(\vec D,\vec B,\vec E,\vec H)$ is given by
\er{guigjgjffghPPNint}. So we get a variational principle related to
Maxwell equations in the form \er{MaxVacFull1bjkgjhjhgjaaaint}. See
also subsection \ref{bhjghjfghfgjkjkjkhjh} for the alternative
Lagrangian of the electromagnetic field.

\subsection{Local gravitational time and Maxwell equations}
Consider an inertial or more generally a non-rotating cartesian
coordinate system $\bf{(*)}$. Then, as before, in this system we
have
\begin{equation}\label{jhhjgjhuiiuint}
\vec v(\vec x,t)=\nabla_{\vec x}Z(\vec x,t),
\end{equation}
where $\vec v$ is the vectorial gravitational potential and $Z$ is a
scalar field. Then define a scalar field $\tau:=\tau(\vec x,t)$ by
the following:
\begin{equation}\label{jhhjgjhuiiuiyint}
\tau(\vec x,t)=t+\frac{1}{c^2}Z(\vec x,t).
\end{equation}
We call the quantity $\tau(\vec x,t)$ by the name local
gravitational time. The name "local" and "gravitational" is quite
clear, since $\tau$ depend on the space and time variables and
derived by characteristic function $Z$ of the gravitational field.
The name "time" will be clarified bellow. Note also that, using
\er{noninchgravortbstrjgghguittu1intmmjhhj} in remark \ref{ugyugg},
one can easily deduce that under the change of inertial coordinate
system $\bf{(*)}$ to $\bf{(**)}$ given by the Galilean
Transformation \er{noninchgravortbstrjgghguittu1int} the local
gravitational time $\tau$ transforms as:
\begin{equation}\label{noninchgravortbstrjgghguittu1intmmjhhjhjhjint}
\tau'\,=\,\tau+\frac{1}{c^2}\vec w\cdot\vec x+\frac{|\vec
w|^2}{2c^2}t \,\approx\, \tau+\frac{1}{c^2}\vec w\cdot\vec x,
\end{equation}
where the last equality in
\er{noninchgravortbstrjgghguittu1intmmjhhjhjhjint} is valid if
$\frac{|\vec w|^2}{c^2}\ll 1$.

Next consider the Maxwell equations in the vacuum of the form
\er{MaxVacFull1bjkgjhjhgjaaaint} and consider a curvilinear change
of variables given by:
\begin{equation}\label{giuuihjghgghjgj78zzrrZZffhhhggygghghjhvbKKint}
\begin{cases}
t'=\tau(\vec x,t):=t+\frac{Z(\vec x,t)}{c^2}\\
\vec x'=\vec x.
\end{cases}
\end{equation}
Then, denoting
\begin{equation}\label{giuuihjghgghjgj78zzrrZZffhhhggygghghjhvbgghhjyuuyKKint}
\begin{cases}
\vec E^*:=\vec D-\frac{1}{c}\,\vec v\times \vec H=\vec
E-\frac{1}{c^2}\vec v\times\left(\vec v\times\vec D\right)
\\
\vec H^*:=\vec B+\frac{1}{c}\,\vec v\times \vec E=\vec
H-\frac{1}{c^2}\vec v\times\left(\vec v\times\vec B\right),
\end{cases}
\end{equation}
by \er{MaxVacFull1bjkgjhjhgjaaaint} we rewrite the Maxwell equations
in the new curvilinear coordinates in the case of time independent
$\vec v$ as:
\begin{equation}\label{MaxMedFullGGffgguiuiouiogghghhgghhjhjhjjgghhgKKint}
\begin{cases}
curl_{\vec x'}\vec H= \frac{4\pi}{c}\vec j+
\frac{1}{c}\frac{\partial \vec E^*}{\partial
t'},\\
div_{\vec x'}\vec E^*= 4\pi\left(\rho+\frac{1}{c^2}\,\vec v\cdot\vec j\right),\\
curl_{\vec x'}\vec E+\frac{1}{c}\frac{\partial\vec H^*}{\partial
t'}=0,\\
div_{\vec x'}\vec H^*=0,\\
\vec E^*=\vec E-\frac{1}{c^2}\vec v\times\left(\vec v\times\vec
D\right)
\\
\vec H^*=\vec H-\frac{1}{c^2}\vec v\times\left(\vec v\times\vec
B\right)\\
\vec E=\vec D-\frac{1}{c}\,\vec v\times \vec B,\\
\vec H=\vec B+\frac{1}{c}\,\vec v\times \vec D,
\end{cases}
\end{equation}
(See section \ref{GIGIGU} for details). In particular, in the
approximation, up to the order $\left(\frac{|\vec v|}{c}\right)^2\ll
1$ we have $\vec E^*\approx\vec E$ and $\vec H^*\approx\vec H$ and
then the approximate Maxwell equations have the form:
\begin{equation}\label{MaxMedFullGGffgguiuiouiogghghhgghhjhjhjjgghhgjhgghhhjkiljklKKint}
\begin{cases}
curl_{\vec x'}\vec H= \frac{4\pi}{c}\vec j+
\frac{1}{c}\frac{\partial \vec E}{\partial
t'},\\
div_{\vec x'}\vec E= 4\pi\left(\rho+\frac{1}{c^2}\,\vec v\cdot\vec j\right),\\
curl_{\vec x'}\vec E+\frac{1}{c}\frac{\partial\vec H}{\partial
t'}=0,\\
div_{\vec x'}\vec H=0,\\
\vec E=\vec D-\frac{1}{c}\,\vec v\times \vec B,\\
\vec H=\vec B+\frac{1}{c}\,\vec v\times \vec D.
\end{cases}
\end{equation}
The first four equations in
\er{MaxMedFullGGffgguiuiouiogghghhgghhjhjhjjgghhgjhgghhhjkiljklKKint}
form a following system of equation:
\begin{equation}\label{MaxMedFullGGffgguiuiouiogghghhgghhjhjhjjgghhgjhgghhhjkiljklKKHHint}
\begin{cases}
curl_{\vec x'}\vec H= \frac{4\pi}{c}\vec j^*+
\frac{1}{c}\frac{\partial \vec E}{\partial
t'},\\
div_{\vec x'}\vec E= 4\pi\rho^*,\\
curl_{\vec x'}\vec E+\frac{1}{c}\frac{\partial\vec H}{\partial
t'}=0,\\
div_{\vec x'}\vec H=0,
\end{cases}
\end{equation}
where
\begin{equation}\label{fgjyhyfgfjjjkint}
\vec j^*:=\vec
j\quad\text{and}\quad\rho^*:=\left(\rho+\frac{1}{c^2}\,\vec
v\cdot\vec j\right)
\end{equation}
The system
\er{MaxMedFullGGffgguiuiouiogghghhgghhjhjhjjgghhgjhgghhhjkiljklKKHHint}
coincides with the classical Maxwell equations of the usual
Electrodynamics. Therefore, given known $\vec v$, $\rho$ and $\vec
j$,
\er{MaxMedFullGGffgguiuiouiogghghhgghhjhjhjjgghhgjhgghhhjkiljklKKHHint}
could be solved as easy as the usual wave equation, for example by
the method of retarded potentials. Then backward to
\er{giuuihjghgghjgj78zzrrZZffhhhggygghghjhvbKKint} change of
variables could be made in order to deduce the electromagnetic
fields in coordinates $(\vec x,t)$. Next note that, since we defined
$t'=\tau$, all the above clarifies the name "time" of the quantity
$\tau$. Finally we would like to note that if we have a motion of
some material body with the place $\vec r(t)$ and the velocity $\vec
u(t):=\frac{d\vec r}{dt}(t)$ and we associate the local
gravitational time $\tau$ with this body then clearly
\begin{equation}\label{fgjyhyfgfjjjkjkkjint}
d\tau\,=\, \left(1+\frac{1}{c^2}\vec u(t)\cdot \vec v\left(\vec
r(t),t\right)\right)\,dt\,\approx\, dt,
\end{equation}
where the last equality in \er{fgjyhyfgfjjjkjkkjint} is valid if we
have
\begin{equation}\label{fgjyhyfgfjjjkintkk}
\left(\frac{|\vec v|}{c}\right)^2\ll
1\quad\text{and}\quad\left(\frac{|\vec u(t)|}{c}\right)^2\ll 1.
\end{equation}
So we can use the local gravitational time $\tau$ in the approximate
calculations instead of the true time $t$.

\subsection{Motion of the particles in the gravitational and electromagnetic fields and invariance of Shr\"{o}dinger and Pauli equations}
Given a classical particle with inertial mass $m$, charge $\sigma$,
place $\vec r(t)$ and velocity $\vec u(t)=\frac{d\vec r}{dt}(t)$ in
the outer gravitational field with the vectorial gravitational
potential $\vec v(\vec x,t)$, the outer electromagnetic field with
vectorial and scalar potentials $\vec A(\vec x,t)$ and $\vec
\Psi(\vec x,t)$, and additional conservative field with scalar
potential $V(\vec x,t)$ we consider a Lagrangian:
\begin{equation}\label{vhfffngghkjgghfjjint}
L_0\left(\frac{d\vec r}{dt},\vec
r,t\right):=\frac{m}{2}\left|\frac{d\vec r}{dt}-\vec v(\vec
r,t)\right|^2-\sigma\left(\Psi(\vec r,t)-\frac{1}{c}\vec A(\vec
r,t)\cdot\frac{d\vec r}{dt}\right)+V(\vec r,t).
\end{equation}
Then this Lagrangian is invariant under the change of non-inertial
coordinate system, given by \er{noninchgravortbstrjgghguittu2int}.
Moreover, we can show that a trajectory $\vec
r(t):[0,T]\to\mathbb{R}^3$ is a critical point of the functional
\begin{equation}\label{btfffygtgyggyijhhkkint}
J_0=\int_0^T L_0\left(\frac{d\vec r}{dt}(t),\vec r(t),t\right)dt.
\end{equation}
if and only if, consistently with
\er{noninchgravortbstrjgghguittu2gjgghhjhghjhjgghgghghghtytythvfghfgghjggint}
and \er{LorenzChlljklljklkkzzzint}, we have
\begin{multline}\label{vhfffngghkjgghggtghjgfhjoyuiyuyhiyyukukyihyuint}
m\frac{d^2\vec r}{dt^2}=m\left(\frac{\partial}{\partial t}\vec
v(\vec r,t)+\nabla_{\vec x}\left(\frac{1}{2}\left|\vec v(\vec
r,t)\right|^2\right)-\frac{d\vec r}{dt}\times curl_{\vec x}\vec
v(\vec r,t)\right)+\nabla_{\vec x}V\left(\vec
r,t\right)\\+\sigma\vec E(\vec r,t)+\frac{\sigma}{c}\frac{d\vec
r}{dt}\times \vec B(\vec r,t),
\end{multline}
where $\vec E$ and $\vec B$ are given by
\er{MaxVacFull1bjkgjhjhgjgjgkjfhjfdghghligioiuittrPPNint}.

Next if we define the generalized momentum of the particle $m$ by
\begin{equation}\label{guytyurtydftyiujhint}
\vec P:=\nabla_{\vec r'}L_0\left(\vec r',\vec r,t\right)=m
\frac{d\vec r}{dt}-m\vec v(\vec r,t)+\frac{\sigma}{c}\vec A(\vec
r,t),
\end{equation}
and consider a Hamiltonian
\begin{equation}\label{vhfffngghkjgghfjjhyjjfgint}
H_0\left(\vec P,\vec r,t\right):=\vec P\cdot\frac{d\vec
r}{dt}-L_0\left(\frac{d\vec r}{dt},\vec r,t\right),
\end{equation}
then we obtain:
\begin{equation}\label{vhfffngghkjgghfjjghghghint} H_0\left(\vec
P,\vec r,t\right)= \vec P\cdot\vec v(\vec r,t)+
\frac{1}{2m}\left|\vec P-\frac{\sigma}{c}\vec A(\vec
r,t)\right|^2+\sigma\left(\Psi(\vec r,t)-\frac{1}{c}\vec A(\vec
r,t)\cdot\vec v(\vec r,t)\right)-V\left(\vec r,t\right).
\end{equation}
See subsections \ref{hggyugyuy} and \ref{hggyugyuyhkhj} for the
generalizations of the Lagrangian and Hamiltonian in the case of
system of $n$ classical particles. See also subsection
\ref{gghfhfhfhfh} for the invariance of the classical statistical
Liouville's equation, arisen from this Hamiltonian, under the change
of non-inertial coordinate system. Moreover, see subsection
\ref{hjgutgyutyyt} for the frame independent formulations of the
conditions of thermodynamical equilibrium and canonical and
micro-canonical statistical ensembles.

Next, assume that our coordinate system is inertial. Then since by
\er{MaxVacFull1ninshtrgravortghhghgjkgghklhjgkghghjjkjhjkkggjkhjkhjjhhfhjhklkhkhjjklzzzyyyhjggjhgghhjhNWNWNWNWNWBWHWPPN222int}
we have the following Hamilton-Jacobi type equation
\begin{equation}
\label{MaxVacFull1ninshtrgravortghhghgjkgghklhjgkghghjjkjhjkkggjkhjkhjjhhfhjhklkhkhjjklzzzyyyhjggjhgghhjhNWNWNWNWNWBWHWPPN222kkint}
\begin{cases}
\vec v=\nabla_{\vec x}Z,\\
\frac{\partial Z}{\partial t}+\frac{1}{2}\left|\nabla_{\vec
x}Z\right|^2=-\Phi,
\end{cases}
\end{equation}
where $Z:=Z(\vec x,t)$ is some scalar field and $\Phi$ is the
Newtonian gravitational potential, we rewrite
\er{vhfffngghkjgghfjjint} as:
\begin{equation}\label{vhfffngghkjgghfjjSYSPNkkjjkkint}
L_0\left(\frac{d\vec r}{dt},\vec r,t\right)=L'_0\left(\frac{d\vec
r}{dt},\vec r,t\right)-m\frac{d}{dt}\left\{Z\left(\vec
r(t),t\right)\right\}.
\end{equation}
where
\begin{equation}\label{vhfffngghkjgghfjjSYSPNkkjjkkhhkkint}
L'_0\left(\frac{d\vec r}{dt},\vec r,t\right):=
\frac{m}{2}\left|\frac{d\vec r}{dt}\right|^2-m\Phi(\vec
r,t)-\sigma\left(\Psi(\vec r,t)-\frac{1}{c}\vec A(\vec
r,t)\cdot\frac{d\vec r}{dt}\right)+V\left(\vec r,t\right)
\end{equation}
(see subsection \ref{hggyugyuy} for details). Note that in the given
inertial coordinate system $L'_0$ coincides with the classical
Lagrangian of motion in the gravitational and electromagnetic
fields. Moreover, we rewrite \er{btfffygtgyggyijhhkkint} as:
\begin{equation}\label{btfffygtgyggyijhhkkSYSPNkkint}
J_0=\int_0^T L'_0\left(\frac{d\vec r}{dt},\vec
r,t\right)dt+m\left(Z\left(\vec r(0),0\right)-Z\left(\vec
r(T),T\right)\right).
\end{equation}
Thus the stationary points of the functional $J_0$ will satisfy the
same Euler-Lagrange equations as the stationary points of the
functional
\begin{equation}\label{btfffygtgyggyijhhkkSYSPNkkllkkint}
J'_0=\int_0^T L'_0\left(\frac{d\vec r}{dt},\vec r,t\right)dt,
\end{equation}
provided that the beginning and the ending points of the trajectory
$\vec r(t)$ are fixed.

Next if we consider the motion of a quantum micro-particle with
inertial mass $m$ and charge $\sigma$ in the outer gravitational
field with the vectorial gravitational potential $\vec v(\vec x,t)$,
the outer electromagnetic field with vectorial and scalar potentials
$\vec A(\vec x,t)$ and $\Psi(\vec x,t)$, and additional conservative
field with potential $V(\vec x,t)$, not taking into account the spin
interaction, then the Shr\"{o}dinger equation for this particle is
\begin{equation}\label{vhfffngghkjgghfjjghghghhjghjgghkghggkghghjghint}
i\hbar\frac{\partial\psi}{\partial t}=\hat H_0\cdot\psi,
\end{equation}
where $\psi:=\psi(\vec x,t)\in\mathbb{C}$ is a wave function and
$\hat H_0$ is the Hamiltonian operator. Thus, since by
\er{vhfffngghkjgghfjjghghghint} the Hermitian Hamiltonian operator
has the form of:
\begin{multline}\label{vhfffngghkjgghfjjghghghhjghjgghkghgghjhggjjkgint}
\hat H_0\cdot\psi= -\frac{i\hbar}{2}div_{\vec x}\left\{\psi\vec
v(\vec x,t)\right\}-\frac{i\hbar}{2}\vec v(\vec
x,t)\cdot\nabla_{\vec x}\psi+
\left\{\frac{1}{2m}\left(-i\hbar\nabla_{\vec x}-\frac{\sigma}{c}\vec
A(\vec x,t)\right)\circ\left(-i\hbar\nabla_{\vec
x}-\frac{\sigma}{c}\vec A(\vec
x,t)\right)\right\}\cdot\psi\\+\sigma\left(\Psi(\vec
x,t)-\frac{1}{c}\vec A(\vec x,t)\cdot\vec v(\vec
x,t)\right)\cdot\psi-V\left(\vec x,t\right)\cdot\psi,
\end{multline}
we rewrite the corresponding Shr\"{o}dinger equation as
\begin{multline}\label{vhfffngghkjgghfjjghghghhjghjgghkghgghjhggjjkgfgdiyfgfjkjgjgint}
i\hbar\left(\frac{\partial\psi}{\partial t}+\vec v\cdot\nabla_{\vec
x}\psi\right)+\frac{i\hbar}{2}\left(div_{\vec x}\vec v\right)\psi=\\
-\frac{\hbar^2}{2m}\Delta_{\vec x}\psi+\frac{
i\hbar\sigma}{2mc}div_{\vec x}\left\{\psi\vec A\right\}+\frac{
i\hbar\sigma}{2mc}\vec A\cdot\nabla_{\vec
x}\psi+\left(\sigma\Psi-\frac{\sigma}{c}\vec A\cdot\vec
v+\frac{\sigma^2}{2mc^2}\left|\vec A\right|^2-V\right)\psi.
\end{multline}
Then we can deduce that, under the change of non-inertial cartesian
coordinate system , given by \er{noninchgravortbstrjgghguittu2int},
the Shr\"{o}dinger equation of the form
\er{vhfffngghkjgghfjjghghghhjghjgghkghgghjhggjjkgfgdiyfgfjkjgjgint}
stays invariant, provided that, under
\er{noninchgravortbstrjgghguittu2int} we have
\begin{equation}\label{vyfgjhgjhvhgghint}
\begin{cases}
\psi'=\psi
\\
V'=V
\\
\vec v'=A(t)\cdot \vec v+A'(t)\cdot\vec x+\frac{d\vec z}{dt}(t)\\
\vec A'=A(t)\cdot\vec A\\
\Psi'-\frac{1}{c}\vec A'\cdot\vec v'=\Psi-\frac{1}{c}\vec A\cdot\vec
v .
\end{cases}
\end{equation}
So the laws of Quantum Mechanics are also invariant in every
non-inertial cartesian coordinate system. Next, assume that we are
in some inertial coordinate system and observe the Newtonian Law of
Gravitation in the form of
\er{MaxVacFull1ninshtrgravortghhghgjkgghklhjgkghghjjkjhjkkggjkhjkhjjhhfhjhklkhkhjjklzzzyyyhjggjhgghhjhNWNWBWHWPPN222int}.
Then, as a consequence, we have
\er{MaxVacFull1ninshtrgravortghhghgjkgghklhjgkghghjjkjhjkkggjkhjkhjjhhfhjhklkhkhjjklzzzyyyhjggjhgghhjhNWNWNWNWNWBWHWPPN222int}
for some scalar field $Z$ and the scalar Newtonian gravitational
potential $\Phi$. Thus denoting
\begin{equation}\label{jhgghfiuihyhyiy}
\psi_1:=e^{\frac{im}{\hbar}Z}\psi,
\end{equation}
we rewrite
\er{vhfffngghkjgghfjjghghghhjghjgghkghgghjhggjjkgfgdiyfgfjkjgjgint}
in the given inertial coordinate system as:
\begin{multline*}
i\hbar\frac{\partial\psi_1}{\partial
t}=-\frac{\hbar^2}{2m}\Delta_{\vec x}\psi_1+\frac{
i\hbar\sigma}{2mc}div_{\vec x}\left\{\psi_1\vec A\right\}
+\frac{ i\hbar\sigma}{2mc}\vec A\cdot\nabla_{\vec x}\psi_1
+\left(\sigma\Psi+\frac{\sigma^2}{2mc^2}\left|\vec
A\right|^2-V+m\Phi\right)\psi_1,
\end{multline*}
which coincides with the classical Shr\"{o}dinger equation for this
case. Note also that by Remark \ref{ugyugg}, equality
\er{jhgghfiuihyhyiy} implies that under the change of coordinate
system given by the Galilean Transformation
\er{noninchgravortbstrjgghguittu1int} the quantity $\psi_1$
transforms as:
\begin{equation}\label{jhgghfhjhjhjhyuyuyint}
\psi'_1:=e^{\frac{im}{\hbar}(\vec w\cdot\vec x+\frac{1}{2}|\vec
w|^2t)}\psi_1,
\end{equation}
provided that $\psi'=\psi$. Moreover, \er{jhgghfhjhjhjhyuyuyint}
coincides with the classical law of transformation of the wave
function, under the Galilean Transformation (see section 17 in
\cite{LL}).

Next, again consider the motion of a quantum micro-particle having
the inertial mass $m$ and the charges $\sigma$ with the given
gravitational and electromagnetical fields with potentials $\vec
v(\vec x,t)$, $\vec A(\vec x,t)$ and $\Psi(\vec x,t)$ and additional
conservative field with potential $V(\vec x,t)$, not taking into the
account the the spin interaction. Then consider a Lagrangian density
$L_0$ defined by
\begin{multline}\label{vhfffngghkjgghDDmmkkkZZintkk}
L_0\left(\psi,\vec x,t\right):=
\frac{i\hbar}{2}\left(\left(\frac{\partial\psi}{\partial t}+\vec
v\cdot\nabla_{\vec
x}\psi\right)\cdot\bar\psi-\psi\cdot\left(\frac{\partial\bar\psi}{\partial
t}+\vec v\cdot\nabla_{\vec
x}\bar\psi\right)\right)-\frac{\hbar^2}{2m}\nabla_{\vec
x}\psi\cdot\nabla_{\vec x}\bar\psi\\
-\frac{\hbar\sigma i}{2mc}\left(\nabla_{\vec
x}\psi\cdot\bar\psi-\psi\cdot\nabla_{\vec x}\bar\psi\right)\cdot\vec
A-\frac{\sigma^2}{2mc^2}\left|\vec A\right|^2\psi\cdot\bar\psi
-\sigma\left(\Psi-\frac{1}{c}\vec v\cdot\vec
A\right)\psi\cdot\bar\psi+V\left(\vec x,t\right)\psi\cdot\bar\psi,
\end{multline}
where $\psi\in \mathbb{C}$ is a wave function. Then, as before, we
can prove that $L_0$ is invariant under the change of inertial or
non-inertial cartesian coordinate system, given by
\er{noninchgravortbstrjgghguittu2int}, provided that we take into
account \er{vyfgjhgjhvhgghint}. Moreover, if we consider a
functional
\begin{equation}\label{btfffygtgyggyDDmmkkkZZintkk}
J_0=\int_0^T\int_{\mathbb{R}^3}L_0\left(\psi,\vec x,t\right)d\vec x
dt,
\end{equation}
Then, by \er{vhfffngghkjgghDDmmkkkZZintkk} we get that the
Euler-Lagrange equation for \er{btfffygtgyggyDDmmkkkZZintkk}
coincides with the Shr\"{o}dinger equation in the form of
\er{vhfffngghkjgghfjjghghghhjghjgghkghgghjhggjjkgfgdiyfgfjkjgjgint}.
Next we would like to note that the Lagrangian density $L_0$,
defined by \er{vhfffngghkjgghDDmmkkkZZintkk} obeys $U(1)$ local
symmetry, i.e. for every scalar field $w:=w(\vec x,t)$ one can
easily deduce that $L_0$ in \er{vhfffngghkjgghDDmmkkkZZintkk} is
invariant under the transformation:
\begin{equation}\label{MaxVacFull1bjkgjhjhgjgjgkjfhjfdghghligioiuittrPPNhjkjhkjgghhjjhjintiiihh}
\begin{cases}
\psi\,\to\,e^{-\frac{i\sigma w}{c\hbar}}\,\psi
\\
\Psi\,\to\,\Psi+\frac{1}{c}\frac{\partial w}{\partial t}\\
\vec A\,\to\,\vec A-\nabla_{\vec x}w\\
\vec v\,\to\,\vec v.
\end{cases}
\end{equation}
See subsection \ref{hggyugyuy1} for the generalizations of all
mentioned above about the Shr\"{o}dinger equation to the case of
system of $n$ quantum particles. Furthermore, see subsection
\ref{hilghjhghfdgdg} for the invariance of the Quantum Liouville's
equation, under the change of non-inertial coordinate system.

Next consider the motion of a spin-half quantum micro-particle with
inertial mass $m$ and the charge $\sigma$ in the outer gravitational
and electromagnetical field with potentials $\vec v(\vec x,t)$,
$\vec A(\vec x,t)$ and $\Psi(\vec x,t)$ and additional conservative
field with potential $V(\vec x,t)$. Since the Hamiltonian for a
macro-particle has the form \er{vhfffngghkjgghfjjghghghint}, we
built the Hermitian Hamiltonian operator, taking into account the
spin interaction as
\begin{multline}\label{vhfffngghkjgghfjjghghghSYShmyuuiiuuhmhmiopoopnniukjhjkkkllkkkZZint}
\hat H_0\cdot\psi= -\frac{\hbar^2}{2m}\Delta_{\vec
x}\psi+\frac{i\hbar\sigma}{2mc}div_{\vec x}\left\{\psi\vec A(\vec
x,t)\right\}+\frac{i\hbar\sigma}{2mc}\nabla_{\vec x}\psi\cdot\vec
A(\vec x,t)+\frac{\sigma^2}{2mc^2}\left|\vec A(\vec
x,t)\right|^2\psi\\+\sigma\left(\Psi(\vec x,t)-\frac{1}{c}\vec
v(\vec x,t)\cdot\vec A(\vec x,t)\right)\psi-V\left(\vec
x,t\right)\psi-\frac{i\hbar}{2}div_{\vec x}\left\{\psi\vec v(\vec
x,t)\right\}-\frac{i\hbar}{2}\nabla_{\vec x}\psi\cdot\vec v(\vec
x,t)
\\-\frac{g\sigma\hbar}{2mc}\vec
S\cdot\left(curl_{\vec x}\vec A(\vec
x,t)\,\psi\right)+\frac{\hbar}{4}\vec S\cdot\left(curl_{\vec x}\vec
v(\vec x,t)\,\psi\right),
\end{multline}
where $\psi(\vec x,t)=\left(\psi_1(\vec x,t),\psi_2(\vec
x,t)\right)\in\mathbb{C}^2$ is a two-component wave function, $\hat
H_0$ is the Hamiltonian operator, $\vec S:=(S_1,S_2,S_3)$,
$$S_1=\left(\begin{matrix}0&1\\1&0\end{matrix}\right),\quad
S_2=\left(\begin{matrix}0&-i\\i&0\end{matrix}\right)\quad
S_3=\left(\begin{matrix}1&0\\0&-1\end{matrix}\right)$$ are Pauli
matrices and $g$ is a constant that depends on the type of the
particle (for electron we have $g=1$). Note that, in addition to the
classical term of the spin-magnetic interaction, we added another
term to the Hamiltonian, namely $\frac{\hbar}{4}\vec
S\cdot\left(curl_{\vec x}\vec v(\vec x,t)\,\psi\right)$. This term
vanishes in every non-rotating and, in particular, in every inertial
coordinate system, however it provides the invariance of the
Shr\"{o}dinger-Pauli equation, under the change of non-inertial
cartesian coordinate system, as can be seen in the following Theorem
\ref{gjghghgghgintintrrZZint}. The Shr\"{o}dinger-Pauli equation for
this particle is
\begin{equation}\label{vhfffngghkjgghfjjghghghhjghjgghkghggkghghjghSYSPNnnkkllkkkZZint}
i\hbar\frac{\partial\psi}{\partial t}=\hat H_0\cdot\psi.
\end{equation}
I.e,
\begin{multline}\label{vhfffngghkjgghfjjghghghSYShmyuuiiuuhmhmiopoopnniukjhjkk;l;lkhjjkihjjhkkkkkjjjZZint}
i\hbar\left(\frac{\partial\psi}{\partial t}+\frac{1}{2}div_{\vec
x}\left\{\psi\vec v(\vec x,t)\right\}+\frac{1}{2}\nabla_{\vec
x}\psi\cdot\vec v(\vec x,t)\right)\\=-\frac{\hbar^2}{2m}\Delta_{\vec
x}\psi+\frac{i\hbar\sigma}{2mc}div_{\vec x}\left\{\psi\vec A(\vec
x,t)\right\}+\frac{i\hbar\sigma}{2mc}\nabla_{\vec x}\psi\cdot\vec
A(\vec x,t)+\frac{\sigma^2}{2mc^2}\left|\vec A(\vec
x,t)\right|^2\psi\\
%
%
%
+\sigma\left(\Psi(\vec x,t)-\frac{1}{c}\vec v(\vec x,t)\cdot\vec
A(\vec x,t)\right)\psi-V\left(\vec
x,t\right)\psi+\frac{\hbar}{2}\vec
S\cdot\left(\left(\frac{1}{2}curl_{\vec x}\vec v(\vec
x,t)-\frac{g\sigma}{mc}curl_{\vec x}\vec A(\vec
x,t)\right)\,\psi\right).
\end{multline}
In subsection \ref{hjjghjghggh} we prove the following:
\begin{theorem}\label{gjghghgghgintintrrZZint}
Consider that the change of some cartesian coordinate system $(*)$
to another cartesian coordinate system $(**)$ is given by
\er{noninchgravortbstrjgghguittu2int},
where $A(t)\in SO(3)$ is a rotation. Next, assume that in the
coordinate system $(**)$ we observe a validity of the
Shr\"{o}dinger-Pauli equation of the form:
\begin{multline}\label{MaxVacFull1ninshtrredPPNintrrZZint}
i\hbar\left(\frac{\partial\psi'}{\partial t'}+\frac{1}{2}div_{\vec
x'}\left\{\psi'\vec v'\right\}+\frac{1}{2}\nabla_{\vec
x'}\psi'\cdot\vec v'\right)=-\frac{\hbar^2}{2m'}\Delta_{\vec
x'}\psi'+\frac{i\hbar\sigma'}{2m'c}div_{\vec x'}\left\{\psi'\vec
A'\right\}+\frac{i\hbar\sigma'}{2m'c}\nabla_{\vec x'}\psi'\cdot\vec
A'\\+\frac{(\sigma')^2}{2m'c^2}\left|\vec A'\right|^2\psi'
%
%
%
+\sigma'\left(\Psi'-\frac{1}{c}\vec v'\cdot\vec
A'\right)\psi'-V'\psi'+\frac{\hbar}{2}\vec
S\cdot\left(\left(\frac{1}{2}curl_{\vec x'}\vec
v'-\frac{g'\sigma'}{m'c}curl_{\vec x'}\vec A'\right)\,\psi'\right),
\end{multline}
where $\psi\in\mathbb{C}^2$. Then in the coordinate system $(*)$ we
have the validity of Shr\"{o}dinger-Pauli equation of the same as
\er{MaxVacFull1ninshtrredPPNintrrZZint} form:
\begin{multline}\label{MaxVacFull1ninshtrhjkkredPPNintrrZZint}
i\hbar\left(\frac{\partial\psi}{\partial t}+\frac{1}{2}div_{\vec
x}\left\{\psi\vec v\right\}+\frac{1}{2}\nabla_{\vec x}\psi\cdot\vec
v\right)=-\frac{\hbar^2}{2m}\Delta_{\vec
x}\psi+\frac{i\hbar\sigma}{2mc}div_{\vec x}\left\{\psi\vec
A\right\}+\frac{i\hbar\sigma}{2mc}\nabla_{\vec x}\psi\cdot\vec
A\\+\frac{\sigma^2}{2mc^2}\left|\vec A\right|^2\psi
%
%
%
+\sigma\left(\Psi-\frac{1}{c}\vec v\cdot\vec
A\right)\psi-V\psi+\frac{\hbar}{2}\vec
S\cdot\left(\left(\frac{1}{2}curl_{\vec x}\vec
v-\frac{g\sigma}{mc}curl_{\vec x}\vec A\right)\,\psi\right).
\end{multline}
provided that
\begin{equation}\label{yuythfgfyftydtydtydtyddyyyhhddhhhredPPN111hgghjgintintrrZZint}
\begin{cases}
g'=g\\
V'=V,\\
\sigma'=\sigma,\\
m'=m,\\
\vec v'=A(t)\cdot \vec v+\frac{dA}{dt}(t)\cdot\vec x+\frac{d\vec z}{dt}(t),\\
\vec A'=A(t)\cdot \vec A,\\
\Psi'-\vec v'\cdot\vec A'=\Psi-\vec v\cdot\vec A,\\
\psi'=U(t)\cdot\psi,
\end{cases}
\end{equation}
where $U(t)\in SU(2)$ is some special unitary $2\times 2$ matrix
i.e. $U(t)\in\mathbb{C}^{2\times 2}$, $det\,U(t)=1$, $U(t)\cdot
U^*(t)=I$ where $U^*(t)$ is the Hermitian adjoint to $U(t)$ matrix:
$U^*(t):=\bar U(t)^T$ and $I$ is the identity $2\times 2$ matrix.
Moreover, $U(t)$ is characterized by the equality:
\begin{equation}\label{gyfyfgfgfgghZZint}
U^*(t)\cdot\vec S\cdot U(t)=A(t)\cdot\vec S.
\end{equation}
\end{theorem}
Next, again consider the motion of a quantum micro-particle with
spin-half, inertial mass $m$ and the charge $\sigma$ with the given
gravitational and electromagnetical fields with potentials $\vec
v(\vec x,t)$, $\vec A(\vec x,t)$ and $\Psi(\vec x,t)$ and additional
conservative field with potential $V(\vec x,t)$, taking into the
account spin interaction. Then consider a Lagrangian density $L$
defined by
\begin{multline}\label{vhfffngghkjgghDDmmkkkZZint}
L\left(\psi,\vec x,t\right):=
\frac{i\hbar}{2}\left(\left(\frac{\partial\psi}{\partial t}+\vec
v\cdot\nabla_{\vec
x}\psi\right)\cdot\bar\psi-\psi\cdot\left(\frac{\partial\bar\psi}{\partial
t}+\vec v\cdot\nabla_{\vec
x}\bar\psi\right)\right)-\frac{\hbar^2}{2m}\nabla_{\vec
x}\psi\cdot\nabla_{\vec x}\bar\psi\\
-\frac{\hbar\sigma i}{2mc}\left(\nabla_{\vec
x}\psi\cdot\bar\psi-\psi\cdot\nabla_{\vec x}\bar\psi\right)\cdot\vec
A-\frac{\sigma^2}{2mc^2}\left|\vec A\right|^2\psi\cdot\bar\psi
-\sigma\left(\Psi-\frac{1}{c}\vec v\cdot\vec
A\right)\psi\cdot\bar\psi\\ -\frac{\hbar}{2}\left(\left(\vec
S\cdot\left(\frac{1}{2}curl_{\vec x}\vec
v-\frac{g\sigma}{mc}curl_{\vec x}\vec
A\right)\right)\cdot\psi\right)\cdot\bar\psi+V\left(\vec
x,t\right)\psi\cdot\bar\psi,
\end{multline}
where $\psi\in \mathbb{C}^2$ is a two-component wave function. Then
similarly to the proof of Theorem \ref{gjghghgghgintintrrZZint} we
can prove that $L$ is invariant under the change of inertial or
non-inertial cartesian coordinate system, given by
\er{noninchgravortbstrjgghguittu2int}, provided that we take into
account
\er{yuythfgfyftydtydtydtyddyyyhhddhhhredPPN111hgghjgintintrrZZint}.
Moreover, if we consider a functional
\begin{equation}\label{btfffygtgyggyDDmmkkkZZint}
J=\int_0^T\int_{\mathbb{R}^3}L\left(\psi,\vec x,t\right)d\vec x dt,
\end{equation}
then, by \er{vhfffngghkjgghDDmmkkkZZint} we get that the
Euler-Lagrange equation for \er{btfffygtgyggyDDmmkkkZZint} coincides
with the Shr\"{o}dinger-Pauli equation in the form of
\er{vhfffngghkjgghfjjghghghSYShmyuuiiuuhmhmiopoopnniukjhjkk;l;lkhjjkihjjhkkkkkjjjZZint}.
Next we would like to note that, as before, the Lagrangian density
$L$, defined by \er{vhfffngghkjgghDDmmkkkZZint} obeys $U(1)$ local
symmetry, i.e. for every scalar field $w:=w(\vec x,t)$ one can
easily deduce that $L$ in \er{vhfffngghkjgghDDmmkkkZZint} is
invariant under the transformation:
\begin{equation}\label{MaxVacFull1bjkgjhjhgjgjgkjfhjfdghghligioiuittrPPNhjkjhkjgghhjjhjintiiihhintsp}
\begin{cases}
\psi\,\to\,e^{-\frac{i\sigma w}{c\hbar}}\,\psi
\\
\Psi\,\to\,\Psi+\frac{1}{c}\frac{\partial w}{\partial t}\\
\vec A\,\to\,\vec A-\nabla_{\vec x}w\\
\vec v\,\to\,\vec v.
\end{cases}
\end{equation}
See subsection \ref{ghghggjghfgjk} for the generalization of the
Shr\"{o}dinger-Pauli equation to the case of a system of $n$
spin-half micro-particles. See also subsection
\ref{hilghjhghfdgdg11} for the Quantum Liouville's equation  of a
system of $n$ spin-half particles.

\subsection{Unified gravitational-electromagnetic field and conservation laws in the case of the Newtonian-type gravity}
Similarly to our assumption that the electromagnetic field is
influenced by gravitational field, we also can assume that the
gravitational field is influenced by electromagnetic field. We
remind that we assume that the first approximation of the law of
gravitation is given by
\er{MaxVacFull1ninshtrgravortghhghgjkgghklhjgkghghjjkjhjkkggjkhjkhjjhhfhjhklkhkhjjklzzzyyyNWBWHWPPN222int}.
However, till now we said nothing about the relation between the
density of inertial and gravitational masses. If $\mu$ is the
density of inertial masses and $M$ is the density of gravitational
masses, then consistently with the classical Newtonian theory of
gravitation we assume that in the absence of essential
electromagnetic fields we should have
\begin{equation}\label{gghjgghfghdint}
M=\mu.
\end{equation}
In order to satisfy the conservation laws of linear and angular
momentums and energy, consider the following conserved scalar field
$Q$, that we call "electromagnetical-gravitational" mass density,
which is negligible in the absence of electromagnetic fields and
satisfies the identity
\begin{equation}
\label{MaxVacFull1ninshtrgravortghhghgjkgghklhjgkghghjjkjhjkkggjkhjkhjjhhfhjhklkhkhjjklzzzyyyNWNWBWHWPPNint}
\frac{\partial Q}{\partial t}+div_{\vec x}\left\{Q\vec v\right\}=-
div_\vec x\left\{\frac{1}{4\pi c}\vec D\times \vec B\right\}
\end{equation}
in the general case. Then, instead of \er{gghjgghfghdint}, for the
general case of gravitational-electromagnetic fields we consider the
following relation between the gravitational and inertial mass
densities
\begin{equation}\label{gghjgghfghdkjgjjint}
M=\mu+Q.
\end{equation}
Then by
\er{MaxVacFull1ninshtrgravortghhghgjkgghklhjgkghghjjkjhjkkggjkhjkhjjhhfhjhklkhkhjjklzzzyyyNWBWHWPPN222int}
and \er{gghjgghfghdkjgjjint} we have the following law of
gravitation:
\begin{equation}
\label{MaxVacFull1ninshtrgravortghhghgjkgghklhjgkghghjjkjhjkkggjkhjkhjjhhfhjhklkhkhjjklzzzyyyNWBWHWPPNint}
\begin{cases}
curl_{\vec x}\left(curl_{\vec x}\vec v\right)= 0,\\
\frac{\partial}{\partial t}\left(div_{\vec x}\vec v\right)+div_{\vec
x}\left\{\left(div_{\vec x}\vec v\right)\vec
v\right\}+\frac{1}{4}\left|d_{\vec x}\vec v+\{d_{\vec x}\vec
v\}^T\right|^2-\left(div_{\vec x}\vec v\right)^2=  -4\pi G(\mu+Q).
\end{cases}
\end{equation}
The laws
\er{MaxVacFull1ninshtrgravortghhghgjkgghklhjgkghghjjkjhjkkggjkhjkhjjhhfhjhklkhkhjjklzzzyyyNWNWBWHWPPNint}
and
\er{MaxVacFull1ninshtrgravortghhghgjkgghklhjgkghghjjkjhjkkggjkhjkhjjhhfhjhklkhkhjjklzzzyyyNWBWHWPPNint}
are invariant under the change of non-inertial cartesian coordinate
system, given by \er{noninchgravortbstrjgghguittu2int}, provided
that, under \er{noninchgravortbstrjgghguittu2int} we have $Q'=Q$ and
$\mu'=\mu$. In particular, in the inertial coordinate system $(*)$
we should have:
\begin{equation}
\label{MaxVacFull1ninshtrgravortghhghgjkgghklhjgkghghjjkjhjkkggjkhjkhjjhhfhjhklkhkhjjklzzzyyyhjggjhgghhjhNWNWBWHWPPNint}
\begin{cases}
curl_{\vec x}\vec v= 0,\\
\frac{\partial\vec v}{\partial t}+d_\vec x\vec v\cdot\vec v=
-\nabla_{\vec x}\Phi,
\end{cases}
\end{equation}
where $\Phi$ is the scalar gravitational potential which is a scalar
field satisfying in every coordinate system:
\begin{equation}
\label{MaxVacFull1ninshtrgravortghhghgjkgghklhjgkghghjjkjhjkkggjkhjkhjjhhfhjhklkhkhjjklzzzyyyhjggjhgghhjhNWNWNWBWHWPPNint}
\Delta_{\vec x}\Phi=4\pi G(\mu+Q).
\end{equation}
\begin{remark}\label{ghghvghhgggh}
Lemma \ref{fgbfghfh} from Appendix gives some insight that the
"electromagnetical-gravitational" mass density $Q$ in
\er{MaxVacFull1ninshtrgravortghhghgjkgghklhjgkghghjjkjhjkkggjkhjkhjjhhfhjhklkhkhjjklzzzyyyNWNWBWHWPPNint}
should have the values of the same order as the quantity
$\frac{1}{c^2}\left(|\vec D|^2+|\vec B|^2\right)$ and therefore, in
the usual circumstances is negligible with respect to the inertial
mass density $\mu$. Thus we can write $Q\approx 0$ in
\er{MaxVacFull1ninshtrgravortghhghgjkgghklhjgkghghjjkjhjkkggjkhjkhjjhhfhjhklkhkhjjklzzzyyyNWBWHWPPNint},
i.e. the force of gravity in an inertial coordinate system
approximately equals to the classical Newtonian force of gravity.
\end{remark}
Next consider the Maxwell equation in the vacuum in the form
\er{MaxVacFull1bjkgjhjhgjaaaint} and consistently with
\er{noninchgravortbstrjgghguittu2gjgghhjhghjhjgghgghghghtytythvfghfgghjggint},
consider the second Law of Newton for the moving continuum with the
inertial mass density $\mu$ and the field of velocities $\vec u$:
\begin{equation}\label{MaxVacFull1ninshtrgravortghhghgjkgghklhjgkghghjjkjhjkkggjkhjkhjjhhfhjhkjkhbbgjhzzzyyykkknnnNWBWHWPPNint}
\mu\frac{\partial\vec u}{\partial t}+\mu d_{\vec x}\vec u\cdot\vec
u=-\mu\vec u\times curl_{\vec x}\vec v+\mu\partial_{t}\vec
v+\mu\nabla_{\vec x}\left(\frac{1}{2}|\vec v|^2\right)+\rho\vec
E+\frac{1}{c}\vec j\times\vec B+\vec G.
\end{equation}
where $\rho\vec E+\frac{1}{c}\vec j\times\vec B$ is the volume
density of the Lorentz force and $\vec G$ is the total volume
density of all non-gravitational and non-electromagnetic forces
acting on the continuum with mass density $\mu$. Then, in section
\ref{GravElectro} we prove that in inertial coordinate systems we
have conservation laws of the linear momentum, the angular momentum
and the energy. More precisely, we have the following theorem:
\begin{theorem}
Consider the Maxwell equation for the vacuum in the form
\er{MaxVacFull1bjkgjhjhgjaaaint} and the second Law of Newton for
the moving continuum in the form
\er{MaxVacFull1ninshtrgravortghhghgjkgghklhjgkghghjjkjhjkkggjkhjkhjjhhfhjhkjkhbbgjhzzzyyykkknnnNWBWHWPPNint}.
Next, assume that in some cartesian coordinate system $(*)$ we
observe the gravitational law in the form of
\er{MaxVacFull1ninshtrgravortghhghgjkgghklhjgkghghjjkjhjkkggjkhjkhjjhhfhjhklkhkhjjklzzzyyyhjggjhgghhjhNWNWBWHWPPNint},
\er{MaxVacFull1ninshtrgravortghhghgjkgghklhjgkghghjjkjhjkkggjkhjkhjjhhfhjhklkhkhjjklzzzyyyhjggjhgghhjhNWNWNWBWHWPPNint}
and
\er{MaxVacFull1ninshtrgravortghhghgjkgghklhjgkghghjjkjhjkkggjkhjkhjjhhfhjhklkhkhjjklzzzyyyNWNWBWHWPPNint}.
Then in the system $(*)$ we have the following laws of conservation
of the linear momentum, angular momentum and energy:
\begin{multline}\label{hvkgkjgkjbjkjjkgjglhhkhjyuyghjhhjhjfghfdhgdhdfdhzzzyyyjffjjkkhhkgjgkjhfhjhgfffjgjhgjffgjggjgjggjhgkkkggjgjgnnnmmmNWNWNWNWBWHWMPPNint}
\frac{\partial}{\partial t}\left(\mu\vec u+Q\vec v+\frac{1}{4\pi
c}\,\vec D\times \vec B\right)=\\-div_\vec x\left\{\mu\vec
u\otimes\vec u+Q\vec v\otimes\vec v+\left(\frac{1}{4\pi c}\vec
D\times \vec B\right)\otimes \vec v+\vec v\otimes\left(\frac{1}{4\pi
c}\vec D\times \vec
B\right)\right\}\\
+\frac{1}{4\pi}div_\vec x\left\{\vec D\otimes \vec D+\vec B\otimes
\vec B-\frac{1}{2}\left(|\vec D|^2+|\vec
B|^2\right)I-\frac{1}{G}\nabla_{\vec x}\Phi\otimes\nabla_{\vec
x}\Phi+\frac{1}{2G}\left|\nabla_{\vec x}\Phi\right|^2 I\right\}+\vec
G,
\end{multline}
\begin{multline}\label{hvkgkjgkjbjkjjkgjglhhkhjyuyghjhhjhjfghfdhgdhdfdhzzzyyyjffjjkkhhkgjgkjhfhjhgfffjgjhgjffgjggjgjggjhgkkkggjgjgmomnnnmmmNWNWBWHWMMPPNint}
\frac{\partial}{\partial t}\left(\vec x\times(\mu\vec u)+\vec
x\times(Q\vec v)+\vec x\times\left(\frac{1}{4\pi c}\,\vec D\times
\vec B\right)\right)=\\-div_\vec x\left\{\mu(\vec x\times\vec
u)\otimes\vec u+Q(\vec x\times\vec v)\otimes\vec v+\left(\vec
x\times\left(\frac{1}{4\pi c}\vec D\times \vec
B\right)\right)\otimes \vec v+(\vec x\times\vec
v)\otimes\left(\frac{1}{4\pi c}\vec D\times \vec
B\right)\right\}\\
+\frac{1}{4\pi}div_\vec x\left\{(\vec x\times\vec D)\otimes \vec
D+(\vec x\times\vec B)\otimes \vec B-\frac{1}{G}(\vec
x\times\nabla_{\vec x}\Phi)\otimes\nabla_{\vec
x}\Phi\right\}\\+\frac{1}{8\pi}curl_{\vec x}\left\{\left(|\vec
D|^2+|\vec B|^2-\frac{1}{G}\left|\nabla_{\vec
x}\Phi\right|^2\right)\vec x\right\}+ \vec x\times\vec G,
\end{multline}
and
\begin{multline}\label{MaxVacFull1ninshtrgravortghhghgjkgghklhjgkghghjjkjhjkkggjkhjkhjjhhfhjhkjkhbbgjhhkjhhklhzzzyyykkkgkhjjhgfhjjffghuikkgkjghhjkjknnnmmmNWNWBWHWhhjhhENPPNint}
\frac{\partial}{\partial t}\left(\frac{1}{2}\mu|\vec
u|^2+\frac{1}{2}Q\left|\vec v\right|^2+\frac{\vec D\cdot\vec E+\vec
B\cdot\vec H}{8\pi}-\frac{1}{8\pi G}\big|\nabla_{\vec
x}\Phi\big|^2\right)=\\
-div_\vec x\left\{\left(\frac{\mu|\vec u|^2}{2}\right)\vec
u+\left(\frac{Q|\vec v|^2}{2}\right)\vec v+\frac{1}{2}\left|\vec
v\right|^2\left(\frac{1}{4\pi c}\,\vec D\times \vec
B\right)+\left(\frac{\vec D\cdot\vec E+\vec B\cdot\vec
H}{8\pi}\right)\vec v\right\}
\\
+\frac{1}{4\pi}div_\vec x\left\{(\vec D\otimes \vec D+ \vec B\otimes
\vec B)\cdot \vec v-\frac{1}{2}\left(|\vec D|^2+|\vec
B|^2\right)\vec v-c \vec D\times \vec B\right\}\\-div_\vec
x\left\{\Phi\left(\mu\vec u+Q\vec v+\frac{1}{4\pi c}\,\vec D\times
\vec B\right)\right\}
-\frac{1}{4\pi G}div_{\vec x}\left\{\Phi\frac{\partial}{\partial
t}(\nabla_{\vec x}\Phi)\right\}+\vec G\cdot\vec u.
\end{multline}
\end{theorem}
%
%
%
%
%
%

Next given known the distribution of inertial mass density of some
continuum medium $\mu:=\mu(\vec x,t)$, the field of velocities of
this medium $\vec u:=\vec u(\vec x,t)$, the charge density
$\rho:=\rho(\vec x,t)$ and the current density $\vec j:=\vec j(\vec
x,t)$ consider a Lagrangian density $L$ for the unified
gravitational-electromagnetic field, defined by
\begin{multline}\label{vhfffngghkjgghDDint}
L\left(\vec A,\Psi,\vec v,\Phi,\vec p,\vec
x,t\right):=\frac{1}{8\pi}\left|-\nabla_{\vec
x}\Psi-\frac{1}{c}\frac{\partial\vec A}{\partial t}+\frac{1}{c}\vec
v\times curl_{\vec x}\vec A\right|^2-\frac{1}{8\pi}\left|curl_{\vec
x}\vec A\right|^2-\left(\rho\Psi-\frac{1}{c}\vec A\cdot\vec
j\right)\\+\frac{\mu}{2}\left|\vec u-\vec v\right|^2+
\frac{1}{2}\left(d_{\vec x}\vec v+\left\{d_{\vec x}\vec
v\right\}^T\right):\left(d_{\vec x}\vec p+\left\{d_{\vec x}\vec
p\right\}^T\right)-2\left(div_{\vec x}\vec v\right)\left(div_{\vec
x}\vec p\right)\\+\frac{1}{4\pi G}\left(div_{\vec x}\vec
v\right)\left(\frac{\partial \Phi}{\partial t}+\vec
v\cdot\nabla_{\vec x} \Phi\right)+\frac{1}{4\pi
G}\Phi\left(div_{\vec x}\vec v\right)^2-\frac{\Phi}{16\pi
G}\left|d_{\vec x}\vec v+\left\{d_{\vec x}\vec
v\right\}^T\right|^2+\frac{1}{8\pi G}\left|\nabla_{\vec
x}\Phi\right|^2,
\end{multline}
where $\Phi$ is an ancillary proper scalar field and $\vec p$ is an
ancillary proper vector field. Then, as before, we can show that $L$
is invariant under the change of non-inertial cartesian coordinate
system given by \er{noninchgravortbstrjgghguittu2int}, provided
that, under \er{noninchgravortbstrjgghguittu2int} we have
\begin{equation}\label{vyfgjhgjhvhgghintintint}
\begin{cases}
\vec p'=A(t)\cdot\vec p\\
\Phi'=\Phi
\\
\vec v'=A(t)\cdot \vec v+A'(t)\cdot\vec x+\frac{d\vec z}{dt}(t)\\
\vec A'=A(t)\cdot\vec A\\
\Psi'-\frac{1}{c}\vec A'\cdot\vec v'=\Psi-\frac{1}{c}\vec A\cdot\vec
v .
\end{cases}
\end{equation}
Then in section \ref{ghjfhgfdhd} we obtain that a configuration
$(\vec A,\Psi,\vec v,\Phi,\vec p)$ is a critical point of the
functional
\begin{equation}\label{btfffygtgyggyDDint}
J=\int_0^T\int_{\mathbb{R}^3}L\left(\vec A,\Psi,\vec v,\Phi,\vec
p,\vec x,t\right)d\vec x dt.
\end{equation}
if and only if it satisfies
\begin{equation}\label{guigjgjffghguygjyfDDint}
\begin{cases}
curl_{\vec x}\vec H=\frac{4\pi}{c}\vec
j+\frac{1}{c}\frac{\partial\vec D}{\partial
t}\\
div_{\vec x}\vec D=4\pi\rho\\
curl_{\vec x}\vec E+\frac{1}{c}\frac{\partial\vec B}{\partial t}=0\\
div_{\vec x}\vec B=0\\
\vec E=\vec D-\frac{1}{c}\vec v\times\vec B\\
\vec H=\vec B+\frac{1}{c}\vec v\times\vec D\\
curl_{\vec x}\left(curl_{\vec x}\vec v\right)=0
\\
\frac{\partial}{\partial t}\left\{div_{\vec x}\vec v\right\}+\vec
v\cdot\nabla_{\vec x}\left(div_{\vec x}\vec
v\right)+\frac{1}{4}\left|d_{\vec x}\vec v+\left\{d_{\vec x}\vec
v\right\}^T\right|^2=-\Delta_{\vec x}\Phi\\
\left(\mu\vec u-\mu\vec v+\frac{1}{4\pi c}\vec D\times\vec
B\right)=curl_{\vec x}\left(curl_{\vec x}\vec p\right)-\frac{1}{4\pi
G}\left(\frac{\partial}{\partial t}\left(\nabla_{\vec
x}\Phi\right)-curl_{\vec x}\left(\vec v\times\nabla_{\vec
x}\Phi\right)+\left(\Delta_{\vec x}\Phi\right)\vec v\right),
\end{cases}
\end{equation}
where, consistently with \er{guigjgjffghPPNint} we denote:
\begin{equation}\label{guigjgjffghDDint}
\begin{cases}
\vec D:=-\nabla_{\vec x}\Psi-\frac{1}{c}\frac{\partial\vec
A}{\partial t}+\frac{1}{c}\vec
v\times curl_{\vec x}\vec A\\
\vec B:=curl_{\vec x}\vec A
\\
\vec E:=-\nabla_{\vec x}\Psi-\frac{1}{c}\frac{\partial\vec A}{\partial t}\\
\vec H:=curl_{\vec x}\vec A+\frac{1}{c}\vec
v\times\left(-\nabla_{\vec x}\Psi-\frac{1}{c}\frac{\partial\vec
A}{\partial t}+\frac{1}{c}\vec v\times curl_{\vec x}\vec A\right).
\end{cases}
\end{equation}
In particular, using continuum equation $\partial_t\mu+div_{\vec
x}\left(\mu\vec u\right)=0$ from the last equality in
\er{guigjgjffghguygjyfDDint} we deduce
\begin{equation*}
\frac{\partial}{\partial t}\left(\frac{1}{4\pi G}\Delta_{\vec
x}\Phi-\mu\right)+div_{\vec x}\left\{\left(\frac{1}{4\pi
G}\Delta_{\vec x}\Phi-\mu\right)\vec v\right\}=-div_{\vec
x}\left\{\frac{1}{4\pi c}\vec D\times\vec B\right\}.
\end{equation*}
Thus denoting $Q=\Delta_{\vec x}\Phi/4\pi G-\mu$ we deduce the
following system of equation for the gravitational-electromagnetic
field, invariant under the change of non-inertial cartesian
coordinate system:
\begin{equation}\label{guigjgjffghguygjyfghggDDint}
\begin{cases}
curl_{\vec x}\vec H=\frac{4\pi}{c}\vec
j+\frac{1}{c}\frac{\partial\vec D}{\partial
t}\\
div_{\vec x}\vec D=4\pi\rho\\
curl_{\vec x}\vec E+\frac{1}{c}\frac{\partial\vec B}{\partial t}=0\\
div_{\vec x}\vec B=0\\
\vec E=\vec D-\frac{1}{c}\vec v\times\vec B\\
\vec H=\vec B+\frac{1}{c}\vec v\times\vec D\\
curl_{\vec x}\left(curl_{\vec x}\vec v\right)=0
\\
\frac{\partial}{\partial t}\left(div_{\vec x}\vec v\right)+div_{\vec
x}\left\{\left(div_{\vec x}\vec v\right)\vec
v\right\}+\frac{1}{4}\left|d_{\vec x}\vec v+\{d_{\vec x}\vec
v\}^T\right|^2-\left(div_{\vec x}\vec v\right)^2=  -4\pi G(\mu+Q)
\\
\frac{\partial Q}{\partial t}+div_{\vec x}\left(Q\vec
v\right)=-div_{\vec x}\left\{\frac{1}{4\pi c}\vec D\times\vec
B\right\},
\end{cases}
\end{equation}
which is consistent with \er{MaxVacFull1bjkgjhjhgjaaaint},
\er{MaxVacFull1ninshtrgravortghhghgjkgghklhjgkghghjjkjhjkkggjkhjkhjjhhfhjhklkhkhjjklzzzyyyNWBWHWPPNint}
and
\er{MaxVacFull1ninshtrgravortghhghgjkgghklhjgkghghjjkjhjkkggjkhjkhjjhhfhjhklkhkhjjklzzzyyyNWNWBWHWPPNint}.

\subsection{Transformations of general scalar and vector fields under the
change of cartesian coordinate system} In order to get the above
results we established some trivial calculus consequences about the
behavior of scalar, vector and matrix fields, under the change of
cartesian coordinate system of the form
\er{noninchgravortbstrjgghguittu2int}. We combine them in the form
of Proposition after the following definition:
\begin{definition}\label{bggghghgjint}
Consider the change of some inertial or non-inertial cartesian
coordinate system $(*)$ to another cartesian coordinate system
$(**)$ of the form \er{noninchgravortbstrjgghguittu2int} where
$A(t)\in SO(3)$ is a rotation.
\begin{itemize}
\item
We say that a general scalar field $\psi:=\psi(\vec
x,t):\R^3\times[t_0,+\infty)\to\R$ is a proper scalar field if,
under every change of coordinate system given by
\er{noninchgravortbstrjgghguittu2int}, this field transforms by the
law:
\begin{equation}\label{uguyytfddddint}
\psi'(\vec x',t')=\psi(\vec x,t).
\end{equation}

\item
We say that a general vector field $\vec f:=\vec f(\vec
x,t):\R^3\times[t_0,+\infty)\to\R^3$ is a proper vector field if,
under every change of coordinate system given by
\er{noninchgravortbstrjgghguittu2int}, this field transforms by the
law:
\begin{equation}\label{uguyytfddddgghjjgint}
\vec f'(\vec x',t')=A(t)\cdot\vec f(\vec x,t),
\end{equation}

\item
We say that a general vector field $\vec v:=\vec v(\vec
x,t):\R^3\times[t_0,+\infty)\to\R^3$ is a speed-like vector field
if, under every change of coordinate system given by
\er{noninchgravortbstrjgghguittu2int}, this field transforms by the
law:
\begin{equation}
\label{NoIn5redbstrint}\vec v'(\vec x',t')=A(t)\cdot \vec v(\vec
x,t)+
\frac{d A}{dt}(t)\cdot\vec x+
\vec w(t),
\end{equation}
where we set
\begin{equation}\label{buitguihjkint}
\vec w(t):=\frac{d\vec z}{dt}(t)\quad\quad\forall\,t.
\end{equation}

\item
We say that a general matrix valued field $T:=T(\vec
x,t):\R^3\times[t_0,+\infty)\to\R^{3\times 3}$ is a proper matrix
field if, under every change of coordinate system given by
\er{noninchgravortbstrjgghguittu2int}, this field transforms by the
law:
\begin{equation}\label{uguyytfddddgghjjghjjjint}
T'(\vec x',t')=A(t)\cdot T(\vec x,t)\cdot A^T(t)=A(t)\cdot T(\vec
x,t)\cdot \left\{A(t)\right\}^{-1}.
\end{equation}
\end{itemize}
\end{definition}
\begin{proposition}\label{yghgjtgyrtrtint}
If $\psi:\R^3\times[t_0,+\infty)\to\R$ is a proper scalar field,
$\vec f:\R^3\times[t_0,+\infty)\to\R^3$ and $\vec
g:\R^3\times[t_0,+\infty)\to\R^3$ are proper vector fields, $\vec
v:\R^3\times[t_0,+\infty)\to\R^3$ and $\vec
u:\R^3\times[t_0,+\infty)\to\R^3$ are speed-like vector fields and
$T:\R^3\times[t_0,+\infty)\to\R^{3\times 3}$ is a proper matrix
field, then:
\begin{itemize}
\item[{\bf(i)}] scalar fields defined in every coordinate system as $\,\vec f\cdot\vec g$, $div_{\vec x} \vec f$ and $div_{\vec x} \vec
v$ are proper scalar fields;

\item[{\bf(ii)}] vector fields defined in every coordinate system as $\nabla_{\vec x}\psi$, $div_{\vec x} T$, $curl_{\vec x}\vec
f$, $\vec f\times\vec g$,
$div_{\vec x}\left(d_{\vec x} \vec v+\left\{d_{\vec x} \vec
v\right\}^T\right)$, $\nabla_{\vec x}\left(div_{\vec x}\vec
v\right)$, $\Delta_{\vec x}\vec v$, $curl_{\vec x}\left(curl_{\vec
x}\vec v\right)$ and $(\vec u-\vec v)$ are proper vector fields;

\item[{\bf(iii)}] matrix fields defined in every coordinate system as $\,d_{\vec x} \vec f$ and $\left(d_{\vec x} \vec v+\left\{d_{\vec x} \vec v\right\}^T\right)$ are proper
matrix fields;

\item[{\bf(iv)}] scalar fields $\xi:\R^3\times[t_0,+\infty)\to\R$ and $\zeta:\R^3\times[t_0,+\infty)\to\R$, defined
in every coordinate system by
\begin{equation}\label{vfyutuyfffhhgfhgfhgtgint}
\xi:=\frac{\partial\psi}{\partial t}+\vec v\cdot\nabla_{\vec x}\vec
\psi\quad \text{and}\quad \zeta:=\frac{\partial\psi}{\partial
t}+div_{\vec x}\left\{ \psi\vec v\right\}
\end{equation}
are proper scalar fields;

\item[{\bf(v)}] vector fields $\vec\Theta:\R^3\times[t_0,+\infty)\to\R^3$ and $\vec\Xi:\R^3\times[t_0,+\infty)\to\R^3$, defined
in every coordinate system by
\begin{equation}\label{vfyutuyfffhhgfhgfhint}
\vec\Theta:=\frac{\partial \vec f}{\partial t}- curl_{\vec
x}\left(\vec v\times \vec f\right)+\left({div}_{\vec x}\vec
f\right)\vec v\quad\text{and}\quad \vec\Xi:=\frac{\partial \vec
f}{\partial t}- \vec v\times curl_{\vec x}\vec f+\nabla_{\vec
x}\left(\vec v\cdot\vec f\right),
\end{equation}
are proper vector fields and
\begin{equation}
\label{vhfffngghhjghhgjlkhjhkPPPint} \vec\Xi
=
\vec\Theta-\left(div_{\vec x}\vec v\right)\vec f+ \left(d_{\vec
x}\vec v+\left\{d_{\vec x}\vec v\right\}^T\right)\cdot\vec f.
\end{equation}
\end{itemize}
\end{proposition}
We prove this Proposition in section \ref{fhfgfghdfdfdfd}.

Next, we also prove the following Proposition in section
\ref{fhfgfghdfdfdfd}.
\begin{proposition}\label{yghgjtgyrtrgjgjgjgjgjgjgjgjgjgjgjgjgj33int}
Consider some fixed inertial or non-inertial cartesian coordinate
system, denoted by the sign $(\{0\})$, and consider a vector field
$\vec d:=\vec d(\vec x,t):\R^3\times[t_0,+\infty)\to\R^3$ associated
\underline{only} with the chosen coordinate system $(\{0\})$. Then
there exist a unique \underline{speed-like} vector field $\vec
v_{\vec d}:\R^3\times[t_0,+\infty)\to\R^{3}$ (associated with any
inertial or non-inertial cartesian coordinate system) such that in
the \underline{particular} system $(\{0\})$ we have:
\begin{equation}\label{vfyutuyfffhhgfhgfh33iy33int}
\vec v_{\vec d}(\vec x,t)=\vec d(\vec x,t)\quad\quad\forall(\vec
x,t) \in\R^3\times[t_0,+\infty)\,.
\end{equation}
Moreover, there exist unique a \underline{proper} vector field $\vec
h_{\vec d}:\R^3\times[t_0,+\infty)\to\R^3$ (associated with any
inertial or non-inertial cartesian coordinate system) such that in
the \underline{particular} system $(\{0\})$ we have:
\begin{equation}\label{vfyutuyfffhhgfhgfh33iyjkkj33int}
\vec h_{\vec d}(\vec x,t)=\vec d(\vec x,t)\quad\quad\forall(\vec
x,t) \in\R^3\times[t_0,+\infty)\,.
\end{equation}
\end{proposition}

We also will need the following Definitions and Propositions and we
prove these Propositions in section \ref{fhfgfghdfdfdfd}.
\begin{definition}\label{jhjkhhjhjhgjgint}
Let $\vec u:=\vec u(\vec x,t)$ be a speed-like vector field defined
for $(\vec x,t)\in\Omega$ with domain of definition
$\Omega\subset\mathbb{R}^3\times\mathbb{R}$ and let $(*)$ be some
inertial or non-inertial cartesian coordinate system. Then $\vec u$
is called \underline{generally trivial speed-like field} if there
exists another cartesian coordinate system $(**)$ such that under
the change of coordinate system $(*)$ to another cartesian
coordinate system $(**)$ given by:
\begin{equation}\label{noninchgravortbstrjgghguittu2hhhhhint}
\begin{cases}
\vec x'=A(t)\cdot\vec x+\vec z(t),\\
t'=t,
\end{cases}
\end{equation}
where $A(t)\in SO(3)$ is a rotation, we have
\begin{equation}
\label{jljjjjkhhhint}A(t)\cdot \vec u(\vec x,t)+
\frac{d A}{dt}(t)\cdot\vec x+
\frac{d\vec z}{dt}(t)=\vec u'(\vec x',t')=0.
\end{equation}
\end{definition}
\begin{proposition}
\label{jhjkhhjhjhgjgjjkjkint} Let $\vec u:=\vec u(\vec x,t)$ be a
smooth speed-like vector field defined for $(\vec x,t)\in\Omega$
with connected domain $\Omega\subset\mathbb{R}^3\times\mathbb{R}$,
such that for every instant of time $\tau$ the domain
$\Omega_\tau=\left\{\vec x\in \mathbb{R}^3:\,(\vec
x,\tau)\in\Omega\right\}$ is also connected. Then $\vec u$ is a
\underline{generally trivial speed-like field} if and only if we
have
\begin{equation}
\label{jljjjjkhhhjhjhjhint} d_{\vec x}\vec u+\left\{d_{\vec x}\vec
u\right\}^T=0\quad\quad\forall (\vec x,t)\in\Omega.
\end{equation}
\end{proposition}
\begin{proposition}
\label{jhjkhhjhjhgjgjjkjkjhhjhint} Let $\vec v:=\vec v(\vec
x,t):\mathbb{R}^3\times[t_0,+\infty)\to\mathbb{R}^3$ be a speed-like
vector field, such that there exists some cartesian coordinate
system ${(*)}$, where we have:
\begin{equation}
\label{jljjjjkhhhjhjhjhljjjljljouukhhint} \lim\limits_{|\vec
x|\to+\infty}\vec v(\vec x,t)=0\quad\quad\forall t\quad\quad\text{(
in the given particular coordinate system ${(*)}$)}\,.
\end{equation}
Then there exist a uniquely defined \underline{generally trivial
speed-like} vector field $\,\vec k:=\vec k(\vec
x,t):\mathbb{R}^3\times[t_0,+\infty)\to\mathbb{R}^3$ and a uniquely
defined \underline{proper} vector field $\,\vec h:=\vec h(\vec
x,t):\mathbb{R}^3\times[t_0,+\infty)\to\mathbb{R}^3$, so that in
every cartesian coordinate system we have
\begin{equation}
\label{jljjjjkhhhjhjhjhljjjljljouukhhjhjgint} \lim\limits_{|\vec
x|\to+\infty}\vec h(\vec x,t)=0\quad\quad\forall t\,,
\end{equation}
and in every cartesian coordinate system we can decompose vector
field $\vec v$ as:
\begin{equation}
\label{jljjjjkhhhjhjhjhljjjljljkjkint} \vec v(\vec x,t)=\vec h(\vec
x,t)+\vec k(\vec x,t)\quad\quad\quad\quad\forall (\vec
x,t)\in\mathbb{R}^3\times[t_0,+\infty).
\end{equation}
\end{proposition}

\begin{definition}\label{jhjkhhjhjhgjggjkkljkljklkhoiuiiyint}
We say that two (inertial or non-inertial) cartesian coordinate
systems $(*)$ and $(**)$ are \underline{equivalent} if the change of
coordinates from the system $(*)$ to the system $(**)$ is given by
\begin{equation}\label{noninchgravortbstrjkjkjggg33jklkhjhjhuiuiint}
\begin{cases}
\vec x'=B\cdot\vec x+\vec c,\\
t'=t,
\end{cases}
\end{equation}
where $B\in SO(3)$ is a constant (independent of time) rotation and
$\vec c\in \mathbb{R}^3$ is a constant (independent of time) vector.
In other words, we obtain system $(**)$ from system $(*)$ just by a
constant (independent of time) rotation in $\mathbb{R}^3$ and/or by
a constant (independent of time) shift of the origin in
$\mathbb{R}^3$.
\end{definition}

\begin{definition}\label{jhjkhhjhjhgjg;kkljklint}
Let $\vec v:=\vec v(\vec x,t):\mathbb{R}^3\times
[t_0,+\infty)\to\mathbb{R}^3$ be a speed-like smooth vector field
and let $(*)$ be some (inertial or non-inertial) cartesian
coordinate system. We say that $\vec v$ is \underline{asymptotically
acceptable} in the coordinate system $(*)$, if there exist
continuous mappings $B(t):[t_0,+\infty)\to\mathbb{R}^{3\times 3}$
and $\vec l(t):[t_0,+\infty)\to\mathbb{R}^{3}$, such that in the
coordinate system $(*)$ we have
\begin{equation}
\label{jljjjjkhhhhkhhjhjkjkjklint}
\begin{cases}
\lim\limits_{|\vec x|\to+\infty}\left(d_{\vec x}\vec v(\vec
x,t)\right)
=B(t)\quad\quad\quad\forall t\in [t_0,+\infty)\,,\\
\lim\limits_{|\vec x|\to+\infty}\left(d_{\vec x}\vec v(\vec
x,t)+\left\{d_{\vec x}\vec v(\vec
x,t)\right\}^T\right)=0\quad\quad\quad\forall t\in [t_0,+\infty)\,,
\\
\lim\limits_{|\vec x|\to+\infty}\left(\vec v(\vec x,t)-B(t)\cdot\vec
x\right)=\vec l(t)\quad\quad\quad\forall t\in [t_0,+\infty)\,.
\end{cases}
\end{equation}
\end{definition}
\begin{proposition}
\label{jhjkhhjhjhgjgjjkjkjklljjint} Let $\vec v:=\vec v(\vec
x,t):\mathbb{R}^3\times [t_0,+\infty)\to\mathbb{R}^3$ be a
speed-like smooth vector field. Then the following three conditions
are \underline{equivalent}:
\begin{itemize}
\item[{\bf (i)}]
There exists a (inertial or non-inertial) cartesian coordinate
system $(*)$  where the vector field $\vec v$ is asymptotically
acceptable.
\item[{\bf (ii)}]
The vector field $\vec v$ is asymptotically acceptable in
\underline{every} (inertial or non-inertial) cartesian coordinate
system.
\item[{\bf (iii)}]
There exists a (inertial or non-inertial) cartesian coordinate
system ($\{0\}$) in which we have \begin{align}
\label{jljjjjkhhhhkhhjhjkjkjkljojkjkjkhj1int} \lim\limits_{|\vec
x|\to+\infty}\left(d_{\vec x}\vec v(\vec x,t)\right)
=0\quad\quad\quad\forall t\in [t_0,+\infty)\,,\\
\label{jljjjjkhhhhkhhjhjkjkjkljojkjkjkhj2int} \lim\limits_{|\vec
x|\to+\infty}\vec v(\vec x,t)=0\quad\quad\quad\forall t\in
[t_0,+\infty)\,.
\end{align}
Moreover, in that case the cartesian coordinate system
where \er{jljjjjkhhhhkhhjhjkjkjkljojkjkjkhj2int} holds, is unique,
up to an equivalence, in other words it is unique, up to a constant
(independent of time) rotation in $\mathbb{R}^3$ and/or up to a
constant (independent of time) shift of the origin in
$\mathbb{R}^3$.
\end{itemize}
\end{proposition}
As a direct consequence of both Propositions
\ref{jhjkhhjhjhgjgjjkjkjhhjhint} and
\ref{jhjkhhjhjhgjgjjkjkjklljjint} we deduce the following Corollary:
\begin{corollary}
\label{jhjkhhjhjhgjgjjkjkjhhjhCCint} Let $\vec v:=\vec v(\vec
x,t):\mathbb{R}^3\times[t_0,+\infty)\to\mathbb{R}^3$ be an
\underline{asymptotically acceptable} speed-like vector field. Then
there exist a uniquely defined \underline{generally trivial
speed-like} vector field $\,\vec k:=\vec k(\vec
x,t):\mathbb{R}^3\times[t_0,+\infty)\to\mathbb{R}^3$ and a uniquely
defined \underline{proper} vector field $\,\vec h:=\vec h(\vec
x,t):\mathbb{R}^3\times[t_0,+\infty)\to\mathbb{R}^3$, so that in
every cartesian coordinate system we have
\begin{equation}
\label{jljjjjkhhhjhjhjhljjjljljouukhhjhjgCCint} \lim\limits_{|\vec
x|\to+\infty}\vec h(\vec x,t)=0\quad\quad\forall t\,,
\end{equation}
and in every cartesian coordinate system we can decompose vector
field $\vec v$ as:
\begin{equation}
\label{jljjjjkhhhjhjhjhljjjljljkjkCCint} \vec v(\vec x,t)=\vec
h(\vec x,t)+\vec k(\vec x,t)\quad\quad\quad\quad\forall (\vec
x,t)\in\mathbb{R}^3\times[t_0,+\infty).
\end{equation}
Moreover, in addition, the proper vector field $\,\vec h$
necessarily satisfies
\begin{equation}
\label{jljjjjkhhhjhjhjhljjjljljouukhhjhjgCCkhjkjhCCint}
\lim\limits_{|\vec x|\to+\infty}\left(d_{\vec x}\vec h(\vec
x,t)\right)=0\quad\quad\forall t\,.
\end{equation}
\end{corollary}

\begin{definition}\label{gjgjffjjjjjjjjjjjjjjjjjjjjjjhhint}
Let $\mathcal{U}$ be a nonempty subset of the set of all inertial
and non-inertial cartesian coordinate systems. We say that
$\mathcal{U}$ is an \underline{extended Galilean group} if, given
arbitrary cartesian coordinate system $(*)$ inside the subset
$\mathcal{U}$, the following statement holds:
\begin{itemize}
\item
The arbitrary cartesian coordinate system $(**)$ belongs to the
subset $\mathcal{U}$ if and only if there exist a constant
(independent of time) rotation $B\in SO(3)$ and constant
(independent of time) vectors $\vec c\in \mathbb{R}^3$ and $\vec
w\in \mathbb{R}^3$, so that the change of coordinates from the
system $(*)$ to the system $(**)$ is given by
\begin{equation}\label{noninchgravortbstrjkjkjggg33jklkhjhjhuiuiljjljjkklint}
\begin{cases}
\vec x'=B\cdot\vec x+\vec c+\vec w t,\\
t'=t.
\end{cases}
\end{equation}
\end{itemize}
In other words, the system $(**)$ belongs to $\mathcal{U}$ if and
only if, up to equivalence of cartesian coordinate systems, the
system $(**)$ can be obtained from the system $(*)$ by the Galilean
Transformation
\begin{equation}\label{noninchgravortbstrjgghguittu1hjmjhint}
\begin{cases}
\vec x'=\vec x+\vec wt,\\
t'=t.
\end{cases}
\end{equation}
\end{definition}
\begin{definition}\label{gjgjffjjjjjjjjjjjjjjjjjjjjjjhhrint} Let
$\mathcal{V}$ be a nonempty subset of the set of all inertial and
non-inertial cartesian coordinate systems. We say that $\mathcal{V}$
is a \underline{non-rotational rigid motion group} if, given
arbitrary cartesian coordinate system $(*)$ inside the subset
$\mathcal{V}$, the following statement holds:
\begin{itemize}
\item
The arbitrary cartesian coordinate system $(**)$ belongs to the
subset $\mathcal{V}$ if and only if there exist a constant
(independent of time) rotation $B\in SO(3)$ and a smooth
vector-valued function of the time variable $\vec z(t):\mathbb{R}\to
\mathbb{R}^3$, so that the change of coordinates from the system
$(*)$ to the system $(**)$ is given by
\begin{equation}\label{noninchgravortbstrjkjkjggg33jklkhjhjhuiuiljjljjkklrint}
\begin{cases}
\vec x'=B\cdot\vec x+\vec z(t),\\
t'=t.
\end{cases}
\end{equation}
\end{itemize}
\end{definition}
\begin{definition}\label{gjgjffjjjjjjjjjjjjjjjjjjjjjjhhrfint}
Let $\mathcal{W}$ be a certain nonempty subset of coordinate
systems. We say that $\mathcal{W}$ is a \underline{general rigid
motion group} if, given arbitrary coordinate system $(*)$ inside the
subset $\mathcal{W}$, the following statement holds:
\begin{itemize}
\item
The arbitrary coordinate system $(**)$ belongs to the subset
$\mathcal{W}$ if and only if there exist a smooth (time dependent)
rotation $A(t):\mathbb{R}\to SO(3)$ and a smooth vector-valued
function of the time variable $\vec z(t):\mathbb{R}\to
\mathbb{R}^3$, so that the change of coordinates from the system
$(*)$ to the system $(**)$ is given by
\begin{equation}\label{noninchgravortbstrjkjkjggg33jklkhjhjhuiuiljjljjkklrfint}
\begin{cases}
\vec x'=A(t)\cdot\vec x+\vec z(t),\\
t'=t.
\end{cases}
\end{equation}
\end{itemize}
In particular, the set of all \underline{cartesian} coordinate
systems is a general rigid motion group.
\end{definition}

\begin{definition}\label{jhjkhhjhjhgjg;kkljkljkkjlhlhhlhint}
Let $\vec v:=\vec v(\vec x,t):\mathbb{R}^3\times
[t_0,+\infty)\to\mathbb{R}^3$ be an \underline{asymptotically
acceptable} speed-like smooth vector field (see Definition
\ref{jhjkhhjhjhgjg;kkljklint}).
\begin{itemize}
\item
We call that the cartesian coordinate system ($\{0\}$) is
\underline{preferable} for the vector field $\vec v$ if in the
system ($\{0\}$) we have
\begin{align} \label{jljjjjkhhhhkhhjhjkjkjkljojkjkjkhj1hhint}
\lim\limits_{|\vec x|\to+\infty}\left(d_{\vec x}\vec v(\vec
x,t)\right)
=0\quad\quad\quad\forall t\in [t_0,+\infty)\,,\\
\label{jljjjjkhhhhkhhjhjkjkjkljojkjkjkhj2hhint} \lim\limits_{|\vec
x|\to+\infty}\vec v(\vec x,t)=0\quad\quad\quad\forall t\in
[t_0,+\infty)\,.
\end{align}
Note that by Proposition \ref{jhjkhhjhjhgjgjjkjkjklljjint} the
unique (up to equivalence) preferable for the vector field $\vec v$
cartesian system ($\{0\}$) exists.
\item
We call that the cartesian coordinate system $(*)$ is
\underline{inertial} with respect to the vector field $\vec v$ if in
the system $(*)$ we have
\begin{align} \label{jljjjjkhhhhkhhjhjkjkjkljojkjkjkhj1hhjint}
\lim\limits_{|\vec x|\to+\infty}\left(d_{\vec x}\vec v(\vec
x,t)\right)
=0\quad\quad\quad\forall t\in [t_0,+\infty)\,,\\
\label{jljjjjkhhhhkhhjhjkjkjkljojkjkjkhj2hhjint} \lim\limits_{|\vec
x|\to+\infty}\vec v(\vec x,t)=\vec d\quad\quad\quad\forall t\in
[t_0,+\infty)\,,
\end{align}
where $\vec d\in \mathbb{R}^3$ is a constant (independent of $t$)
vector. Note that, in particular, by the definition, the preferable
for $\vec v$ cartesian system is also an inertial with respect to
$\vec v$  coordinate system.
\item
We call that the cartesian coordinate system $(*)$ is
\underline{non-rotating} with respect to the vector field $\vec v$
if in the system $(*)$ we have
\begin{equation} \label{jljjjjkhhhhkhhjhjkjkjkljojkjkjkhj1hhjhkkjint}
\lim\limits_{|\vec x|\to+\infty}\left(d_{\vec x}\vec v(\vec
x,t)\right)=0\quad\quad\quad\forall t\in [t_0,+\infty)\,,
\end{equation}
(without specifying $\lim\limits_{|\vec x|\to+\infty}\vec v(\vec
x,t)$). Note that, in particular, by the definition, every inertial
with respect to $\vec v$ cartesian coordinate system is also a
non-rotating with respect to $\vec v$ coordinate system.
\end{itemize}
\end{definition}
\begin{proposition}\label{jghfgjdhdhgint}
Let $\vec v:=\vec v(\vec x,t):\mathbb{R}^3\times
[t_0,+\infty)\to\mathbb{R}^3$ be an \underline{asymptotically
acceptable} speed-like smooth vector field (see Definition
\ref{jhjkhhjhjhgjg;kkljklint}). Furthermore, let $\mathcal{U}$ be
the set of all \underline{inertial} with respect to $\vec v$
cartesian coordinate systems and let $\mathcal{V}$ be the set of all
\underline{non-rotating} with respect to $\vec v$ cartesian
coordinate systems (see Definition
\ref{jhjkhhjhjhgjg;kkljkljkkjlhlhhlhint}). Then $\mathcal{U}$ is an
\underline{extended Galilean group} and  $\mathcal{V}$ is a
\underline{non-rotational rigid motion group} (see Definitions
\ref{gjgjffjjjjjjjjjjjjjjjjjjjjjjhhint} and
\ref{gjgjffjjjjjjjjjjjjjjjjjjjjjjhhrint}).
\end{proposition}

The next Proposition describes dynamics of the line, the surface and
the volume integrals.
\begin{proposition}\label{jgjjjjjjjjjjjjjjjjjjjjjjjjjjjjjjjjjkhhkhint}
\begin{itemize}
\item[{{\bf (i)}}]
Consider the moving continuum medium with velocity field $\vec
u(\vec x,t)$ and let
$\mathcal{S}:=\mathcal{S}(t)\subset\mathbb{R}^3$ be a
two-dimensional surface oriented by the unit normal $\vec n:=\vec
n(\vec x,t)$ and moving together with the given medium. Then, given
a vector field $\vec f(\vec x,t)$ we have:
\begin{equation}\label{MaxVacFullPPNffGGvbccgcu8uuuuukjjkhhhihjhjkhhkmgjhjhhkhjhhjjhjhjjh11ljkjkljkljjjkjkljkjkint}
\frac{d}{dt}\left(\iint\vec f\cdot\vec n\,
d\mathcal{S}(t)\right)=\iint\left(\frac{\partial \vec f}{\partial
t}- curl_{\vec x}\left(\vec u\times \vec f\right)+\left({div}_{\vec
x}\vec f\right)\vec u\right)\cdot\vec n\, d\mathcal{S}(t).
\end{equation}

\item[{{\bf (ii)}}]
Consider the moving continuum medium with velocity field $\vec
u(\vec x,t)$ and let $\gamma:=\gamma(t)\subset\mathbb{R}^3$ be a
one-dimensional curve oriented by the unit tangent vector $\vec
t:=\vec t(\vec x,t)$ and moving together with the given medium.
Then, given a vector field $\vec f(\vec x,t)$ we have:
\begin{equation}\label{MaxVacFullPPNffGGvbccgcu8uuuuukjjkhhhihjhjkhhkmgjhjhhkhjhhjjhjhjjh11ljkjkljkljjjkjkljkjk11int}
\frac{d}{dt}\left(\int\vec f\cdot\vec t\,
d\gamma(t)\right)=\int\left(\frac{\partial \vec f}{\partial t}- \vec
u\times curl_{\vec x}\vec f+\nabla_{\vec x}\left(\vec u\cdot\vec
f\right)\right)\cdot\vec t\, d\gamma(t).
\end{equation}

\item[{{\bf (iii)}}]
Consider the moving continuum medium with velocity field $\vec
u(\vec x,t)$ and let $\Omega:=\Omega(t)\subset\mathbb{R}^3$ be a
three-dimensional domain moving together with the given medium.
Then, given a scalar field $\psi(\vec x,t)$ we have:
\begin{equation}\label{MaxVacFullPPNffGGvbccgcu8uuuuukjjkhhhihjhjkhhkmgjhjhhkhjhhjjhjhjjh11ljkjkljkljjjkjkljkjk22int}
\frac{d}{dt}\left(\iiint\psi\,
d\Omega(t)\right)=\iiint\left(\frac{\partial \psi}{\partial
t}+{div}_{\vec x}\left\{\psi \vec u\right\}\right)\, d\Omega(t).
\end{equation}
\end{itemize}
\end{proposition}
We prove this Proposition in section \ref{fhfgfghdfdfdfd}. Note,
however, that the result similar to Proposition
\ref{jgjjjjjjjjjjjjjjjjjjjjjjjjjjjjjjjjjkhhkhint} was obtained
before by Cornille in section 19.2 of \cite{PC}.

\subsection{Genuine gravity and inertia}\label{jlhkgkhfgjdfhdint}
\begin{definition}\label{jjhgyftdtdint}
In general, both in the case of the simplest model of the Newtonian
gravity and in the case of any other alternative model of gravity,
we \underline{always} assume that the vectorial gravitational
potential $\vec v:=\vec v(\vec x,t)$ is an asymptotically acceptable
speed-like vector field (see Definition
\ref{jhjkhhjhjhgjg;kkljklint} and Proposition
\ref{jhjkhhjhjhgjgjjkjkjklljjint}). Therefore, Corollary
\ref{jhjkhhjhjhgjgjjkjkjhhjhCCint} implies that there exist a
uniquely defined \underline{generally trivial speed-like} vector
field $\,\vec k:=\vec k(\vec x,t)\in\mathbb{R}^3$ (see Definition
\ref{jhjkhhjhjhgjgint}) and a uniquely defined \underline{proper}
vector field $\,\vec h:=\vec h(\vec x,t)\in\mathbb{R}^3$, so that in
every cartesian coordinate system we have
\begin{equation}
\label{jljjjjkhhhjhjhjhljjjljljouukhhjhjgCChhhhint}
\lim\limits_{|\vec x|\to+\infty}\vec h(\vec x,t)=0\quad\quad\forall
t\quad\quad\text{and}\quad\quad\lim\limits_{|\vec
x|\to+\infty}\left(d_{\vec x}\vec h(\vec
x,t)\right)=0\quad\quad\forall t\,,
\end{equation}
and in every cartesian coordinate system we can decompose vector
field $\vec v$ as:
\begin{equation}
\label{jljjjjkhhhjhjhjhljjjljljkjkCChhhhint} \vec v(\vec x,t)=\vec
h(\vec x,t)+\vec k(\vec x,t)\quad\quad\quad\quad\forall (\vec x,t).
\end{equation}
Then we call $\vec h$ \underline{vectorial potential of genuine
gravity} and $\vec k$ \underline{vectorial potential of inertia}. In
particular, it is clear that, given cartesian coordinate system
$(*)$, this system will be \underline{inertial} if and only if the
vectorial potential of inertia $\vec k$ is a \underline{constant}
(independent of $(\vec x,t)$) in the coordinate system $(*)$ (see
subsection \ref{jlhkgkhfgjdfhd} for the details).
\end{definition}
The name of genuine gravity and inertia is clarified by the fact
that since $\vec k$ is \underline{generally trivial} speed-like
vector field, by the definition, we can completely eliminate $\vec
k$ (getting $\vec k=0$) by the simple change of cartesian coordinate
system (it is fictitious). On the other hand, if $\vec h$ is
nontrivial, since $\vec h$ is a \underline{proper} vector field, we
cannot make it trivial in any other coordinate system (it is
genuine). Moreover, the vector of inertia $\vec k$ depends only on
the coordinate system in the space and it is completely independent
on the physical matter or physical fields filling this space. In
contrast vectorial potential of genuine gravity $\vec h$ depends
essentially on the surrounding physical matter (in the model of the
Newtonian gravity through gravitational masses).

Next, assume the model of Newtonian gravity, given by
\er{MaxVacFull1ninshtrgravortghhghgjkgghklhjgkghghjjkjhjkkggjkhjkhjjhhfhjhklkhkhjjklzzzyyyNWBWHWPPN222int},
and assume that $(*)$ is some inertial cartesian coordinate system.
Then, by
\er{MaxVacFull1ninshtrgravortghhghgjkgghklhjgkghghjjkjhjkkggjkhjkhjjhhfhjhklkhkhjjklzzzyyyhjggjhgghhjhNWNWBWHWPPN222kgghjghjghjint}
in the system $(*)$ we have:
\begin{equation}
\label{MaxVacFull1ninshtrgravortghhghgjkgghklhjgkghghjjkjhjkkggjkhjkhjjhhfhjhklkhkhjjklzzzyyyhjggjhgghhjhNWNWBWHWPPN222hhhhint}
\begin{cases}
curl_{\vec x}\vec v= 0,\\
\frac{\partial\vec v}{\partial t}+\frac{1}{2}\nabla_{\vec
x}\left(|\vec v|^2\right)= -\nabla_{\vec x}\Phi,
\end{cases}
\end{equation}
where $\Phi$ is a scalar Newtonian gravitational potential which is
assumed to be a proper scalar, which satisfies
\begin{equation}
\label{MaxVacFull1ninshtrgravortghhghgjkgghklhjgkghghjjkjhjkkggjkhjkhjjhhfhjhklkhkhjjklzzzyyyhjggjhgghhjhNWNWNWBWHWPPN222hhhhint}
\begin{cases}
\Delta_{\vec x}\Phi=4\pi GM\quad\quad\quad\forall \,(\vec x,t),\\
\lim\limits_{|\vec x|\to+\infty}\nabla_{\vec x}\Phi(\vec
x,t)=0\quad\quad\quad\forall t,
\end{cases}
\end{equation}
with $M$ being the gravitational mass density and $G$ being the
gravitational constant. Thus inserting
\er{jljjjjkhhhjhjhjhljjjljljkjkCChhhhint} into
\er{MaxVacFull1ninshtrgravortghhghgjkgghklhjgkghghjjkjhjkkggjkhjkhjjhhfhjhklkhkhjjklzzzyyyhjggjhgghhjhNWNWBWHWPPN222hhhhint}
and using the fact that in the inertial system $(*)$ the vector
$\vec k$ is a constant, we deduce
\begin{equation}
\label{MaxVacFull1ninshtrgravortghhghgjkgghklhjgkghghjjkjhjkkggjkhjkhjjhhfhjhklkhkhjjklzzzyyyhjggjhgghhjhNWNWBWHWPPN222hhhhkjhkhkjkhhjkkhhkjkljkkhkhint}
\begin{cases}
curl_{\vec x}\vec h= 0,\\
\frac{\partial\vec h}{\partial t}-\vec v\times curl_{\vec x}\vec
h+\nabla_{\vec x}\left(\vec v\cdot\vec h\right)-\nabla_{\vec
x}\left(\frac{1}{2}\left|\vec h\right|^2\right)=-\nabla_{\vec
x}\Phi,
\end{cases}
\end{equation}
that we can alternatively rewrite as:
\begin{equation}
\label{MaxVacFull1ninshtrgravortghhghgjkgghklhjgkghghjjkjhjkkggjkhjkhjjhhfhjhklkhkhjjklzzzyyyhjggjhgghhjhNWNWBWHWPPN222hhhhkjhkhkjkhhjkkhhkjkljkkhkh11int}
\begin{cases}
curl_{\vec x}\vec h= 0,\\
\frac{\partial\vec h}{\partial t}-\vec k\times curl_{\vec x}\vec
h+\nabla_{\vec x}\left(\vec k\cdot\vec h\right)+\nabla_{\vec
x}\left(\frac{1}{2}\left|\vec h\right|^2\right)=-\nabla_{\vec
x}\Phi.
\end{cases}
\end{equation}
However, since $\vec h$ is a proper vector field, by Proposition
\ref{yghgjtgyrtrtint} we deduce, that both
\er{MaxVacFull1ninshtrgravortghhghgjkgghklhjgkghghjjkjhjkkggjkhjkhjjhhfhjhklkhkhjjklzzzyyyhjggjhgghhjhNWNWBWHWPPN222hhhhkjhkhkjkhhjkkhhkjkljkkhkhint}
and
\er{MaxVacFull1ninshtrgravortghhghgjkgghklhjgkghghjjkjhjkkggjkhjkhjjhhfhjhklkhkhjjklzzzyyyhjggjhgghhjhNWNWBWHWPPN222hhhhkjhkhkjkhhjkkhhkjkljkkhkh11int}
preserve their form also in \underline{non-inertial} cartesian
coordinate systems, provided $\Phi$ is a proper scalar, which is
defined in \underline{every} inertial or non-inertial cartesian
coordinate system by the following:
\begin{equation}
\label{MaxVacFull1ninshtrgravortghhghgjkgghklhjgkghghjjkjhjkkggjkhjkhjjhhfhjhklkhkhjjklzzzyyyhjggjhgghhjhNWNWNWBWHWPPN222hhhhjkhkjhjhjhjklint}
\begin{cases}
\Delta_{\vec x}\Phi=4\pi GM\quad\quad\quad\forall \,(\vec x,t),\\
\lim\limits_{|\vec x|\to+\infty}\nabla_{\vec x}\Phi(\vec
x,t)=0\quad\quad\quad\forall t.
\end{cases}
\end{equation}
Finally, by Proposition \ref{jhjkhhjhjhgjgjjkjkint} vector field
$\vec k$ satisfies
\begin{equation}
\label{MaxVacFull1ninshtrgravortghhghgjkgghklhjgkghghjjkjhjkkggjkhjkhjjhhfhjhklkhkhjjklzzzyyyhjggjhgghhjhNWNWNWBWHWPPN222hhhhjkhkjhjhjhhkhhjhjkjhmjljklint}
d_{\vec x}\vec k+\left\{d_{\vec x}\vec
k\right\}^T=0\quad\quad\quad\forall \,(\vec x,t)\,,
\end{equation}
in \underline{every} inertial or non-inertial coordinate system.
Therefore the Newtonian gravity in the form
\er{MaxVacFull1ninshtrgravortghhghgjkgghklhjgkghghjjkjhjkkggjkhjkhjjhhfhjhklkhkhjjklzzzyyyNWBWHWPPN222int},
can be rewritten in the terms of genuine gravity and inertia as
\er{MaxVacFull1ninshtrgravortghhghgjkgghklhjgkghghjjkjhjkkggjkhjkhjjhhfhjhklkhkhjjklzzzyyyhjggjhgghhjhNWNWNWBWHWPPN222hhhhjkhkjhjhjhjklint}
and either
\er{MaxVacFull1ninshtrgravortghhghgjkgghklhjgkghghjjkjhjkkggjkhjkhjjhhfhjhklkhkhjjklzzzyyyhjggjhgghhjhNWNWBWHWPPN222hhhhkjhkhkjkhhjkkhhkjkljkkhkhint}
or, alternatively,
\er{MaxVacFull1ninshtrgravortghhghgjkgghklhjgkghghjjkjhjkkggjkhjkhjjhhfhjhklkhkhjjklzzzyyyhjggjhgghhjhNWNWBWHWPPN222hhhhkjhkhkjkhhjkkhhkjkljkkhkh11int},
that are complemented with
\er{MaxVacFull1ninshtrgravortghhghgjkgghklhjgkghghjjkjhjkkggjkhjkhjjhhfhjhklkhkhjjklzzzyyyhjggjhgghhjhNWNWNWBWHWPPN222hhhhjkhkjhjhjhhkhhjhjkjhmjljklint}
and \er{jljjjjkhhhjhjhjhljjjljljkjkCChhhhint}. Moreover, note again
that
\er{MaxVacFull1ninshtrgravortghhghgjkgghklhjgkghghjjkjhjkkggjkhjkhjjhhfhjhklkhkhjjklzzzyyyhjggjhgghhjhNWNWNWBWHWPPN222hhhhjkhkjhjhjhjklint},
\er{MaxVacFull1ninshtrgravortghhghgjkgghklhjgkghghjjkjhjkkggjkhjkhjjhhfhjhklkhkhjjklzzzyyyhjggjhgghhjhNWNWBWHWPPN222hhhhkjhkhkjkhhjkkhhkjkljkkhkhint},
\er{MaxVacFull1ninshtrgravortghhghgjkgghklhjgkghghjjkjhjkkggjkhjkhjjhhfhjhklkhkhjjklzzzyyyhjggjhgghhjhNWNWBWHWPPN222hhhhkjhkhkjkhhjkkhhkjkljkkhkh11int},
\er{MaxVacFull1ninshtrgravortghhghgjkgghklhjgkghghjjkjhjkkggjkhjkhjjhhfhjhklkhkhjjklzzzyyyhjggjhgghhjhNWNWNWBWHWPPN222hhhhjkhkjhjhjhhkhhjhjkjhmjljklint},
and \er{jljjjjkhhhjhjhjhljjjljljkjkCChhhhint} are invariant under
the change of inertial or non-inertial cartesian coordinate system.
%
%
%
%
%
%

Next note that, since we have $curl_{\vec x}\vec h=0$, we can write
either
\er{MaxVacFull1ninshtrgravortghhghgjkgghklhjgkghghjjkjhjkkggjkhjkhjjhhfhjhklkhkhjjklzzzyyyhjggjhgghhjhNWNWBWHWPPN222hhhhkjhkhkjkhhjkkhhkjkljkkhkhint}
or
\er{MaxVacFull1ninshtrgravortghhghgjkgghklhjgkghghjjkjhjkkggjkhjkhjjhhfhjhklkhkhjjklzzzyyyhjggjhgghhjhNWNWBWHWPPN222hhhhkjhkhkjkhhjkkhhkjkljkkhkh11int}
as the following Hamilton-Jacobi type equation:
\begin{equation}
\label{MaxVacFull1ninshtrgravortghhghgjkgghklhjgkghghjjkjhjkkggjkhjkhjjhhfhjhklkhkhjjklzzzyyyhjggjhgghhjhNWNWNWNWNWBWHWPPN222jujhkhhjhgjggint}
\begin{cases}
\vec h=\nabla_{\vec x}Y,\\
\frac{\partial Y}{\partial t}+\vec v\cdot\nabla_{\vec
x}Y-\frac{1}{2}\left|\nabla_{\vec x}Y\right|^2=\frac{\partial
Y}{\partial t}+\vec k\cdot\nabla_{\vec
x}Y+\frac{1}{2}\left|\nabla_{\vec x}Y\right|^2=-\Phi,
\end{cases}
\end{equation}
where $Y:=Y(\vec x,t)$ is some scalar field. Then, since $\vec h$ is
a proper vector field and $\Phi$ is proper scalar field,  by
Proposition \ref{yghgjtgyrtrtint} we deduce, that
\er{MaxVacFull1ninshtrgravortghhghgjkgghklhjgkghghjjkjhjkkggjkhjkhjjhhfhjhklkhkhjjklzzzyyyhjggjhgghhjhNWNWNWNWNWBWHWPPN222jujhkhhjhgjggint}
is invariant under the changes of inertial or non-inertial cartesian
coordinate systems, provided that we consider $Y$ to be a
\underline{proper} scalar field. Note here that, in contrast to the
fact that the quantity $Z$ and Hamilton-Jacobi equation
\er{MaxVacFull1ninshtrgravortghhghgjkgghklhjgkghghjjkjhjkkggjkhjkhjjhhfhjhklkhkhjjklzzzyyyhjggjhgghhjhNWNWNWNWNWBWHWPPN222int}
were defined only in inertial coordinate systems, the quantity $Y$
and Hamilton-Jacobi equation
\er{MaxVacFull1ninshtrgravortghhghgjkgghklhjgkghghjjkjhjkkggjkhjkhjjhhfhjhklkhkhjjklzzzyyyhjggjhgghhjhNWNWNWNWNWBWHWPPN222jujhkhhjhgjggint}
are defined in every inertial or non-inertial coordinate system.
Next in every \underline{inertial} cartesian coordinate system
define the scalar field $X:=X(\vec x,t)$ by
\begin{equation}
\label{MaxVacFull1ninshtrgravortghhghgjkgghklhjgkghghjjkjhjkkggjkhjkhjjhhfhjhklkhkhjjklzzzyyyhjggjhgghhjhNWNWNWNWNWBWHWPPN222jujhkhhjhgjggljkyiyhkhint}
X(\vec x,t)=Z(\vec x,t)-Y(\vec x,t)\quad\quad\forall(\vec x,t)\,.
\end{equation}
Then, by the first equation in
\er{MaxVacFull1ninshtrgravortghhghgjkgghklhjgkghghjjkjhjkkggjkhjkhjjhhfhjhklkhkhjjklzzzyyyhjggjhgghhjhNWNWNWNWNWBWHWPPN222int}
and the first equation in
\er{MaxVacFull1ninshtrgravortghhghgjkgghklhjgkghghjjkjhjkkggjkhjkhjjhhfhjhklkhkhjjklzzzyyyhjggjhgghhjhNWNWNWNWNWBWHWPPN222jujhkhhjhgjggint},
together with \er{jljjjjkhhhjhjhjhljjjljljkjkCChhhhint}, we deduce
\begin{equation}
\label{jljjjjkhhhjhjhjhljjjljljkjkCChhhhjjjkkhint} \vec k(\vec x,t)=
\nabla_{\vec x}X(\vec x,t)\quad\quad\quad\quad\forall (\vec x,t),
\end{equation}
in every inertial cartesian coordinate system.

\subsection{Lagrangian of the unified
Gravitational-Electromagnetic field in the case of some possible
alternative model of the gravity}\label{ghjfhgfdhdnw22mmint}
Consider $\vec k$ to be the vectorial potential of the inertia,
which is a generally trivial speed-like vector field, assumed to be
fixed in every fixed inertial or non-inertial cartesian coordinate
system (see Definition \ref{jjhgyftdtdint}). Given known the
distribution of inertial mass density of some continuum medium
$\mu:=\mu(\vec x,t)$, the field of velocities of this medium $\vec
u:=\vec u(\vec x,t)$, the charge density $\rho:=\rho(\vec x,t)$ and
the current density $\vec j:=\vec j(\vec x,t)$ consider a Lagrangian
density $L$ defined by
\begin{multline}\label{vhfffngghkjgghDDnw22mmint}
L\left(\vec A,\Psi,\vec v,\Phi_0,\vec h,\vec x,t\right):=\\
\frac{1}{8\pi}\left|-\nabla_{\vec
x}\Psi-\frac{1}{c}\frac{\partial\vec A}{\partial t}+\frac{1}{c}\vec
v\times curl_{\vec x}\vec A\right|^2-\frac{1}{8\pi}\left|curl_{\vec
x}\vec A\right|^2-\left(\rho\Psi-\frac{1}{c}\vec A\cdot\vec
j\right)+\frac{\mu}{2}\left|\vec u-\vec v\right|^2
\\
-\frac{c^2}{8\pi G}\left|-\nabla_{\vec x}\Phi_0
-\frac{1}{c}\frac{\partial\vec h}{\partial t}
+\frac{1}{c}\vec v\times curl_{\vec x}\vec h-\frac{1}{c}\nabla_{\vec
x}\left(\vec h\cdot\vec v\right)\right|^2+\frac{c^2}{8\pi
G}\left|curl_{\vec x}\vec h\right|^2\\+\frac{c^2\beta}{4\pi
G}\left(\frac{1}{4}\left|d_{\vec x}\vec v+\left\{d_{\vec x}\vec
v\right\}^T\right|^2-\left|\Div_{\vec x}\vec v\right|^2\right) ,
\end{multline}
where
\begin{equation}\label{btfffygtgyggyDDnw22jkhk22mmint}
\vec h=\vec v-\vec
k\quad\text{and}\quad\Phi_0=-\frac{1}{2c}\left|\vec v-\vec
k\right|^2,
\end{equation}
and $\beta\in \mathbb{R}$ is some constant. In other words we have:
\begin{multline}\label{vhfffngghkjgghDDnw22jkjkjkjk22gughhg22mmint}
L\left(\vec A,\Psi,\vec v,\Phi_0,\vec h,\vec
x,t\right)=L_1\left(\vec A,\Psi,\vec v,\vec
x,t\right):=\\
\frac{1}{8\pi}\left|-\nabla_{\vec
x}\Psi-\frac{1}{c}\frac{\partial\vec A}{\partial t}+\frac{1}{c}\vec
v\times curl_{\vec x}\vec A\right|^2-\frac{1}{8\pi}\left|curl_{\vec
x}\vec A\right|^2-\left(\rho\Psi-\frac{1}{c}\vec A\cdot\vec
j\right)+\frac{\mu}{2}\left|\vec u-\vec v\right|^2
\\
-\frac{c^2}{8\pi G}\left|
\frac{1}{c}\nabla_{\vec x}\left(\frac{1}{2}\left|\vec v-\vec
k\right|^2\right)-\frac{1}{c}\frac{\partial}{\partial t}\left(\vec
v-\vec k\right)
+\frac{1}{c}\vec v\times curl_{\vec x}\left(\vec v-\vec
k\right)-\frac{1}{c}\nabla_{\vec x}\left(\left(\vec v-\vec
k\right)\cdot\vec v\right)\right|^2\\+\frac{c^2}{8\pi
G}\left|curl_{\vec x}\left(\vec v-\vec
k\right)\right|^2+\frac{c^2\beta}{4\pi
G}\left(\frac{1}{4}\left|d_{\vec x}\vec v+\left\{d_{\vec x}\vec
v\right\}^T\right|^2-\left|\Div_{\vec x}\vec v\right|^2\right).
\end{multline}
In particular, in the inertial coordinate system where $d_{\vec
x}\vec k=0$ and $\partial_t\vec k=0$ we have:
\begin{multline}\label{vhfffngghkjgghDDnw22jkjkjkjk22gughhg22hkhjhjhjh22kljlkjk22mmaint}
L\left(\vec A,\Psi,\vec v,\Phi_0,\vec h,\vec
x,t\right)=L_1\left(\vec A,\Psi,\vec v,\vec
x,t\right):=\\
\frac{1}{8\pi}\left|-\nabla_{\vec
x}\Psi-\frac{1}{c}\frac{\partial\vec A}{\partial t}+\frac{1}{c}\vec
v\times curl_{\vec x}\vec A\right|^2-\frac{1}{8\pi}\left|curl_{\vec
x}\vec A\right|^2-\left(\rho\Psi-\frac{1}{c}\vec A\cdot\vec
j\right)+\frac{\mu}{2}\left|\vec u-\vec v\right|^2
\\
-\frac{c^2}{8\pi G}\left|
-\frac{1}{c}\frac{\partial\vec v}{\partial t}
+\frac{1}{c}\vec v\times curl_{\vec x}\vec v-\frac{1}{c}\nabla_{\vec
x}\left(\frac{1}{2}\left|\vec
v\right|^2\right)\right|^2+\frac{c^2}{8\pi G}\left|curl_{\vec x}\vec
v\right|^2\\+\frac{c^2\beta}{4\pi G}\left(\frac{1}{4}\left|d_{\vec
x}\vec v+\left\{d_{\vec x}\vec v\right\}^T\right|^2-\left|\Div_{\vec
x}\vec v\right|^2\right),
\end{multline}
Note here the advantage of inertial coordinate systems, where $L$
and $L_1$ are complectly independent of the vectorial potential of
the inertia $\vec k$. Furthermore,
we rewrite
\er{vhfffngghkjgghDDnw22jkjkjkjk22gughhg22hkhjhjhjh22kljlkjk22mmaint}
as:
\begin{multline}\label{vhfffngghkjgghDDnw22jkjkjkjk22gughhg22hkhjhjhjh22kljlkjk22mmint}
L\left(\vec A,\Psi,\vec v,\Phi_0,\vec h,\vec
x,t\right)=L_1\left(\vec A,\Psi,\vec v,\vec
x,t\right):=\\
\frac{1}{8\pi}\left|-\nabla_{\vec
x}\Psi-\frac{1}{c}\frac{\partial\vec A}{\partial t}+\frac{1}{c}\vec
v\times curl_{\vec x}\vec A\right|^2-\frac{1}{8\pi}\left|curl_{\vec
x}\vec A\right|^2-\left(\rho\Psi-\frac{1}{c}\vec A\cdot\vec
j\right)+\frac{\mu}{2}\left|\vec u-\vec v\right|^2
\\
-\frac{c^2}{8\pi G}\left|
-\frac{1}{c}\left(\frac{\partial\vec v}{\partial t}+d_{\vec x}\vec
v\cdot\vec v\right)\right|^2+\frac{c^2}{8\pi G}\left|curl_{\vec
x}\vec v\right|^2+\frac{c^2\beta}{4\pi
G}\left(\frac{1}{4}\left|d_{\vec x}\vec v+\left\{d_{\vec x}\vec
v\right\}^T\right|^2-\left|\Div_{\vec x}\vec v\right|^2\right).
\end{multline}
Then, using Proposition \ref{yghgjtgyrtrtint} by
\er{vhfffngghkjgghDDnw22mmint}, \er{btfffygtgyggyDDnw22jkhk22mmint}
and \er{vhfffngghkjgghDDnw22jkjkjkjk22gughhg22mmint} we deduce that
$L$ and $L_1$ are invariant under the change of inertial or
non-inertial cartesian coordinate system, provided that $\vec h$ and
$\vec A$ are proper vector fields, $\vec v$ is a speed-like vector
field and $\Phi_0$ and $\Psi_0:=\Psi-\frac{1}{c}\vec A\cdot\vec v$
are proper scalar fields.

We investigate critical points of the functional
\begin{equation}\label{btfffygtgyggyDDnw22mmint}
J=\int_0^T\int_{\mathbb{R}^3}L_1\left(\vec A,\Psi,\vec v,\vec
x,t\right)d\vec x dt.
\end{equation}
We denote
\begin{equation}\label{guigjgjffghDDnw22mmint}
\begin{cases}
\Psi_0=\Psi-\frac{1}{c}\vec A\cdot\vec v\\
\vec D=-\nabla_{\vec x}\Psi-\frac{1}{c}\frac{\partial\vec
A}{\partial t}+\frac{1}{c}\vec v\times curl_{\vec x}\vec
A=-\nabla_{\vec x}\Psi_0-\frac{1}{c}\frac{\partial\vec A}{\partial
t}+\frac{1}{c}\vec
v\times curl_{\vec x}\vec A-\frac{1}{c}\nabla_{\vec x}\left(\vec A\cdot\vec v\right)\\
\vec B=curl_{\vec x}\vec A
\\
\vec E=-\nabla_{\vec x}\Psi-\frac{1}{c}\frac{\partial\vec A}{\partial t}=\vec D-\frac{1}{c}\vec v\times\vec B\\
\vec H=curl_{\vec x}\vec A+\frac{1}{c}\vec
v\times\left(-\nabla_{\vec x}\Psi-\frac{1}{c}\frac{\partial\vec
A}{\partial t}+\frac{1}{c}\vec v\times curl_{\vec x}\vec
A\right)=\vec B+\frac{1}{c}\vec v\times\vec D.
\end{cases}
\end{equation}
and
\begin{equation}\label{guigjgjffghDDnwjkhhkhkhjklh22mmint}
\begin{cases}
\vec R=-\nabla_{\vec x}\Phi_0-\frac{1}{c}\frac{\partial\vec
h}{\partial t}
+\frac{1}{c}\vec v\times curl_{\vec x}\vec
h-\frac{1}{c}\nabla_{\vec x}\left(\vec h\cdot\vec v\right)\\
\vec Q=curl_{\vec x}\vec h,
\end{cases}
\end{equation}
where $\vec h$ is a proper vector field and $\Phi_0$ is a proper
scalar field that are given by \er{btfffygtgyggyDDnw22jkhk22mmint}.
In other words,
\begin{equation}\label{guigjgjffghDDnwjkhhkhkhjklh22ljjljkhk22mmint}
\begin{cases}
\vec R=\frac{1}{c}\nabla_{\vec x}\left(\frac{1}{2}\left|\vec v-\vec
k\right|^2\right)-\frac{1}{c}\frac{\partial}{\partial t}\left(\vec
v-\vec k\right)
+\frac{1}{c}\vec v\times curl_{\vec x}\left(\vec v-\vec
k\right)-\frac{1}{c}\nabla_{\vec x}\left(\left(\vec v-\vec
k\right)\cdot\vec v\right)\\
\vec Q=curl_{\vec x}(\vec v-\vec k),
\end{cases}
\end{equation}
and in inertial coordinate system where $d_{\vec x}\vec k=0$ and
$\partial_t\vec k=0$ we also have:
\begin{equation}\label{guigjgjffghDDnwjkhhkhkhjklh22ljjljkhk22hhmmint}
\begin{cases}
\vec R=-\frac{1}{c}\left(\frac{\partial\vec v}{\partial t}+d_{\vec
x}\vec
v\cdot\vec v\right)\\
\vec Q=curl_{\vec x}\vec v.
\end{cases}
\end{equation}
As before, by \er{guigjgjffghDDnw22mmint} and
\er{guigjgjffghDDnwjkhhkhkhjklh22mmint} and Proposition
\ref{yghgjtgyrtrtint} we infer that both $\vec D$, $\vec B$ and
$\vec R$, $\vec Q$ are proper vector fields. Then in section
\ref{ghjfhgfdhdnw22mm} we obtain that a configuration $\left(\vec
A,\Psi,\vec v,\vec x,t\right)$ is a critical point of the functional
\begin{equation}\label{btfffygtgyggyDDnw22mmjhkkhhkhint}
J=\int_0^T\int_{\mathbb{R}^3}L_1\left(\vec A,\Psi,\vec v,\vec
x,t\right)d\vec x dt,
\end{equation}
if and only if it satisfies:
\begin{equation}\label{guigjgjffghguygjyfDDnw22mmdint}
\begin{cases}
curl_{\vec x}\vec B=\frac{4\pi}{c}\vec
j+\frac{1}{c}\frac{\partial\vec D}{\partial t}-\frac{1}{c}curl_{\vec
x}\left(\vec v\times \vec D\right)
\\
div_{\vec x}\vec D=4\pi\rho
\\
curl_{\vec x}\vec D+\frac{1}{c}\frac{\partial\vec B}{\partial t}
-\frac{1}{c}\,curl_{\vec x}\left(\vec v\times\vec B\right)=0\\
div_{\vec x}\vec B=0\\
curl_{\vec x}\vec R+\frac{1}{c}\frac{\partial\vec Q}{\partial t}-\frac{1}{c}curl_{\vec x}\left(\vec v\times\vec Q\right)=0\\
div_{\vec x}\vec Q=0
\\
\frac{1}{c}\left(\frac{\partial}{\partial t}\left(\Div_{\vec x}\vec
v\right)+\vec v\cdot\nabla_{\vec x}\left(\Div_{\vec x}\vec
v\right)+\frac{1}{4}\left|d_{\vec x}\vec v+\left\{d_{\vec x}\vec
v\right\}^T\right|^2\right)=\frac{1}{2c}\left|\vec
Q\right|^2-\Div_{\vec x}\vec R
\\
curl_{\vec x}(curl_{\vec x}\vec v)=curl_{\vec x}\vec Q
\\
\frac{4\pi G}{c^2}\left(\mu\vec u-\mu\vec v+\frac{1}{4\pi c}\vec
D\times\vec
B-\frac{c}{4\pi G}\vec R\times\vec Q\right)=\\
(1+\beta)\,curl_{\vec x}\vec Q
-\frac{1}{c}\left(\frac{\partial\vec R}{\partial t}-curl_{\vec
x}\left\{\vec v\times\vec R\right\}+\left(div_{\vec x}\vec
R\right)\vec v\right),
\end{cases}
\end{equation}
and by \er{guigjgjffghDDnwjkhhkhkhjklh22ljjljkhk22hhmmint} in the
inertial frame we have:
\begin{equation}\label{guigjgjffghguygjyfDDnw22khhjhgg22mmdint}
\begin{cases}
curl_{\vec x}\vec B=\frac{4\pi}{c}\vec
j+\frac{1}{c}\frac{\partial\vec D}{\partial t}-\frac{1}{c}curl_{\vec
x}\left(\vec v\times \vec
D\right)\\
div_{\vec x}\vec D=4\pi\rho\\
curl_{\vec x}\vec D+\frac{1}{c}\frac{\partial\vec B}{\partial t}
-\frac{1}{c}\,curl_{\vec x}\left(\vec v\times\vec B\right)=0\\
div_{\vec x}\vec B=0\\
\vec R=-\frac{1}{c}\left(\frac{\partial\vec v}{\partial t}+d_{\vec
x}\vec
v\cdot\vec v\right)\\
\vec Q=curl_{\vec x}\vec v
\\
\frac{4\pi G}{c^2}\left(\mu\vec u-\mu\vec v+\frac{1}{4\pi c}\vec
D\times\vec
B-\frac{c}{4\pi G}\vec R\times\vec Q\right)=\\
(1+\beta)\,curl_{\vec x}\vec Q
-\frac{1}{c}\left(\frac{\partial\vec R}{\partial t}-curl_{\vec
x}\left\{\vec v\times\vec R\right\}+\left(div_{\vec x}\vec
R\right)\vec v\right).
\end{cases}
\end{equation}
Furthermore,  taking $\Div_{\vec x}$ of the both sides of the last
equality in \er{guigjgjffghguygjyfDDnw22mmdint} and using continuum
equation $\partial_t\mu+\Div_{\vec x}\left(\mu\vec u\right)=0$ we
deduce
\begin{multline}\label{guigjgjffghguygjyfDDnw22mmzint}
-\left(\partial_t\mu+\Div_{\vec x}\left(\mu\vec
v\right)\right)+\Div_{\vec x}\left\{\frac{1}{4\pi c}\vec D\times\vec
B-\frac{c}{4\pi  G}\vec R\times\vec Q\right\}=\\ \Div_{\vec
x}\left(\mu\vec u-\mu\vec v+\frac{1}{4\pi c}\vec D\times\vec
B-\frac{c}{4\pi  G}\vec R\times\vec Q\right)=\\
-\frac{c}{4\pi  G}\left(\frac{\partial}{\partial t}\left(\Div_{\vec
x}\vec R\right)+\Div_{\vec x}\left\{\left(div_{\vec x}\vec
R\right)\vec v\right\}\right),
\end{multline}
Therefore, considering the proper scalar quantity $Q_0$, that we
call the field mass, which satisfies
\begin{equation}\label{jtgiktiutiurtfirfmmhfffmmint}
Q_0:=-\mu+\frac{c}{4\pi  G}\,\Div_{\vec x}\vec R,
\end{equation}
by \er{guigjgjffghguygjyfDDnw22mmzint} we deduce
\begin{equation}\label{jtgiktiutiurtfirfmmint}
\frac{\partial Q_0}{\partial t}+div_{\vec x}\left(Q_0\vec
v\right)=-\Div_{\vec x}\left\{\frac{1}{4\pi c}\vec D\times\vec
B-\frac{c}{4\pi  G}\vec R\times\vec Q\right\}.
\end{equation}
Thus, we rewrite \er{guigjgjffghguygjyfDDnw22mmdint} as:
\begin{equation}\label{guigjgjffghguygjyfDDnw22mmint}
\begin{cases}
curl_{\vec x}\vec B=\frac{4\pi}{c}\left(\vec j-\rho\vec
v\right)+\frac{1}{c}\frac{\partial\vec D}{\partial
t}-\frac{1}{c}curl_{\vec x}\left(\vec v\times \vec
D\right)+\frac{1}{c}\left(div_{\vec x}\vec D\right)\vec v,
\\
div_{\vec x}\vec D=4\pi\rho,
\\
curl_{\vec x}\vec D+\frac{1}{c}\frac{\partial\vec B}{\partial t}
-\frac{1}{c}\,curl_{\vec x}\left(\vec v\times\vec B\right)=0,\\
div_{\vec x}\vec B=0,\\
curl_{\vec x}\vec R+\frac{1}{c}\frac{\partial\vec Q}{\partial t}-\frac{1}{c}curl_{\vec x}\left(\vec v\times\vec Q\right)=0,\\
div_{\vec x}\vec Q=0,
\\
\frac{4\pi G}{c^2}\left(\mu\vec u-\mu\vec v+\frac{1}{4\pi c}\vec
D\times\vec
B-\frac{c}{4\pi G}\vec R\times\vec Q\right)=\\
(1+\beta)\,curl_{\vec x}\vec Q
-\frac{1}{c}\left(\frac{\partial\vec R}{\partial t}-curl_{\vec
x}\left\{\vec v\times\vec R\right\}+\left(div_{\vec x}\vec
R\right)\vec v\right),\\
\Div_{\vec x}\vec R=\frac{4\pi G}{c}(\mu+Q_0),
\\
\frac{1}{c}\left(\frac{\partial}{\partial t}\left(\Div_{\vec x}\vec
v\right)+\vec v\cdot\nabla_{\vec x}\left(\Div_{\vec x}\vec
v\right)+\frac{1}{4}\left|d_{\vec x}\vec v+\left\{d_{\vec x}\vec
v\right\}^T\right|^2\right)=\frac{1}{2c}\left|\vec
Q\right|^2-\Div_{\vec x}\vec R,
\\
curl_{\vec x}(curl_{\vec x}\vec v)=curl_{\vec x}\vec Q,
\\
\frac{\partial Q_0}{\partial t}+div_{\vec x}\left(Q_0\vec
v\right)=-\Div_{\vec x}\left\{\frac{1}{4\pi c}\vec D\times\vec
B-\frac{c}{4\pi  G}\vec R\times\vec Q\right\},
\end{cases}
\end{equation}
and we rewrite \er{guigjgjffghguygjyfDDnw22khhjhgg22mmdint} in the
inertial frame as:
\begin{equation}\label{guigjgjffghguygjyfDDnw22khhjhgg22mmint}
\begin{cases}
curl_{\vec x}\vec B=\frac{4\pi}{c}\vec
j+\frac{1}{c}\frac{\partial\vec D}{\partial t}-\frac{1}{c}curl_{\vec
x}\left(\vec v\times \vec
D\right),\\
div_{\vec x}\vec D=4\pi\rho,\\
curl_{\vec x}\vec D+\frac{1}{c}\frac{\partial\vec B}{\partial t}
-\frac{1}{c}\,curl_{\vec x}\left(\vec v\times\vec B\right)=0,\\
div_{\vec x}\vec B=0,\\
\vec R=-\frac{1}{c}\left(\frac{\partial\vec v}{\partial t}+d_{\vec
x}\vec
v\cdot\vec v\right),\\
\vec Q=curl_{\vec x}\vec v,
\\
\frac{4\pi G}{c^2}\left(\mu\vec u-\mu\vec v+\frac{1}{4\pi c}\vec
D\times\vec
B-\frac{c}{4\pi G}\vec R\times\vec Q\right)=\\
(1+\beta)\,curl_{\vec x}\vec Q
-\frac{1}{c}\left(\frac{\partial\vec R}{\partial t}-curl_{\vec
x}\left\{\vec v\times\vec R\right\}+\left(div_{\vec x}\vec
R\right)\vec v\right),\\
\Div_{\vec x}\vec R=\frac{4\pi G}{c} (\mu+Q_0),\\
\frac{\partial Q_0}{\partial t}+div_{\vec x}\left(Q_0\vec
v\right)=-\Div_{\vec x}\left\{\frac{1}{4\pi c}\vec D\times\vec
B-\frac{c}{4\pi  G}\vec R\times\vec Q\right\}.
\end{cases}
\end{equation}
As before, by Proposition \ref{yghgjtgyrtrtint} we deduce that
\er{guigjgjffghguygjyfDDnw22mmint} is invariant under the change of
inertial or non-inertial cartesian coordinate systems. Moreover,
\er{guigjgjffghguygjyfDDnw22khhjhgg22mmint} is invariant under the
change of inertial cartesian coordinate systems. Note also that,
both, \er{guigjgjffghguygjyfDDnw22mmint} in an arbitrary inertial or
non-inertial cartesian coordinate system and
\er{guigjgjffghguygjyfDDnw22khhjhgg22mmint} in an arbitrary inertial
cartesian coordinate system, are complectly independent of the
vectorial potential of the inertia $\vec k$.

Finally, note that in the case of large constant $|\beta|\gg 1$ we
have $\vec Q\to 0$ in \er{guigjgjffghguygjyfDDnw22mmint} and thus,
the gravity equations \er{guigjgjffghguygjyfDDnw22mmint} reduce to
the equations of the Newtonian-type Gravity in the form of
\er{guigjgjffghguygjyfghggDDint}. In that case the gravity field
propagates with the infinite speed. On the other hand, in the case
of vanishing constant $\beta=0$ the form of equations for $\vec R$
and $\vec Q$ in \er{guigjgjffghguygjyfDDnw22mmint} is completely the
same as the form of the Maxwell equations for $\vec D$ and $\vec B$
in \er{guigjgjffghguygjyfDDnw22mmint}, except of the different
meaning of "charges" and "currents" in these two sets of equations.
In that case the electromagnetic and the gravity fields propagate
with the same speed. However, in the mixed case of constant
$\beta\sim 1$ the electromagnetic and the gravity fields propagate
with different finite speeds.

\subsection{Covariant formulation of the physical laws in the
four-dimensional non-relativistic space-time} In Section \ref{CVFRM}
we present the covariant (tensor) formulations of the Maxwell
Equations and the Lagrangian density of the electromagnetic field
and the covariant form of the Lagrangian of motion of charged
particles in the outer gravitational and electromagnetic fields.

\subsubsection{Four-vectors, four-covectors and tensors in the
four-dimensional non-relativistic space-time} First of all we would
like to remind the definitions of the vectors, covectors and
covariant and contravariant tensors of second order in
$\mathbb{R}^4$.
\begin{definition}
Given $\mathcal{S}$, that is a certain subgroup of the group of all
smooth non-degenerate invertible transformations from $\mathbb{R}^4$
onto $\mathbb{R}^4$ having the form
\begin{equation}\label{fgjfjhgghyuyyuint}
\begin{cases}
x'^0=f^{(0)}(x^0,x^1,x^2,x^3),\\
x'^1=f^{(1)}(x^0,x^1,x^2,x^3),\\
x'^2=f^{(2)}(x^0,x^1,x^2,x^3),\\
x'^3=f^{(3)}(x^0,x^1,x^2,x^3),
\end{cases}
\end{equation}
we say that a one-component field $a:=a(x^0,x^1,x^2,x^3)$ is a
scalar field on the group $\mathcal{S}$, if under the coordinate
transformation in the group $\mathcal{S}$ of the form
\er{fgjfjhgghyuyyuint} this field transforms as:
\begin{equation}\label{fgjfjhgghhgjgihhiint}
a'=a.
\end{equation}
Next we say that a four-component field $(a^0,a^1,a^2,a^3)$ is a
four-vector field on the group $\mathcal{S}$, if under the
coordinate transformation in the group $\mathcal{S}$ of the form
\er{fgjfjhgghyuyyuint} every of four components of this field
transforms as:
\begin{equation}\label{fgjfjhgghhgjgint}
a'^j=\sum_{k=0}^{3}\frac{\partial f^{(j)}}{\partial
x^k}a^k\quad\quad\forall j=0,1,2,3.
\end{equation}
Next we say that a four-component field $(a_0,a_1,a_2,a_3)$ is a
four-covector field on the group $\mathcal{S}$, if under the
coordinate transformation in the group $\mathcal{S}$ of the form
\er{fgjfjhgghyuyyuint} every of four components of this field
transforms as:
\begin{equation}\label{fgjfjhgghhgjghjhjint}
a_j=\sum_{k=0}^{3}\frac{\partial f^{(k)}}{\partial
x^j}a'_k\quad\quad\forall j=0,1,2,3.
\end{equation}
Furthermore, we say that a $16$-component field
$\{a_{mn}\}_{m,n=0,1,2,3}$ is a two times covariant tensor field on
the group $\mathcal{S}$, if under the coordinate transformation in
the group $\mathcal{S}$ of the form \er{fgjfjhgghyuyyuint} every of
$16$ components of this field transforms as:
\begin{equation}\label{fgjfjhgghhgjghjhjkkkkjjkint}
a_{mn}=\sum_{j=0}^{3}\sum_{k=0}^{3}\frac{\partial f^{(k)}}{\partial
x^m}\frac{\partial f^{(j)}}{\partial x^n}a'_{kj}\quad\quad\forall\,
m,n=0,1,2,3.
\end{equation}
Next we say that a $16$-component field $\{a^{mn}\}_{m,n=0,1,2,3}$
is a two times contravariant tensor field on the group
$\mathcal{S}$, if under the coordinate transformation in the group
$\mathcal{S}$ of the form \er{fgjfjhgghyuyyuint} every of $16$
components of this field transforms as:
\begin{equation}\label{fgjfjhgghhgjghjhjkkkkggghint}
a'^{mn}=\sum_{j=0}^{3}\sum_{k=0}^{3}\frac{\partial f^{(m)}}{\partial
x^k}\frac{\partial f^{(n)}}{\partial x^j}a^{kj}\quad\quad\forall\,
m,n=0,1,2,3.
\end{equation}
\end{definition}
Next consider the four-dimensional space-time $\mathbb{R}^4$, such
that for every point in space $\vec x=(x_1,x_2,x_3)\in\mathbb{R}^3$
and every instant of time $t$ we correspond the point
$(x^0,x^1,x^2,x^3)\in\mathbb{R}^4$ that has the form:
\begin{equation}\label{fgjfjhgghint}
(x^0,x^1,x^2,x^3):=(ct,x_1,x_2,x_3)=\left(ct,\vec x\right),
\end{equation}
where $c$ is the universal constant in Maxwell equations for vacuum.
In this space we denote by $\mathcal{S}_0$, the subgroup of the
group of smooth non-degenerate invertible mappings, containing
transformations of the form
\begin{equation}\label{noninchgravortbstrjgghguittu2intrrrZZygjyghhjint}
\begin{cases}
x'^0=x^0
\\
x'^j=\sum\limits_{k=1}^{3}A_{jk}\left(\frac{x^0}{c}\right)\,x_k+z_j\left(\frac{x^0}{c}\right)\quad\forall
j=1,2,3,
\end{cases}
\end{equation}
where $$\left\{A_{jk}(t)\right\}_{j,k=1,2,3}=A(t):\mathbb{R}\to
SO(3)$$ is a rotation, smoothly dependent on $t$ and
$$\left(z_1(t),z_2(t),z_3(t)\right)=\vec
z(t):\mathbb{R}\to\mathbb{R}^3$$ also smoothly dependent on $t$.
Then in the terms of time $t$ and three-dimensional space we rewrite
\er{noninchgravortbstrjgghguittu2intrrrZZygjyghhjint} as
\er{noninchgravortbstrjgghguittu2int}, i.e.:
\begin{equation}\label{noninchgravortbstrjgghguittu2intrrrZZygjygint}
\begin{cases}
\vec x'=A(t)\cdot\vec x+\vec z(t),\\
t'=t,
\end{cases}
\end{equation}
where $A(t)\in SO(3)$ is a rotation. I.e. the group $\mathcal{S}_0$
represents all transformations of cartesian non-inertial coordinate
systems in the non-relativistic space-time. It can be easily checked
by trivial calculations that $\mathcal{S}_0$ is indeed a group, i.e.
for every two transformations $f,g\in\mathcal{S}_0$ the composition
$g\circ f$ and the inverse transformation $f^{(-1)}$ are also
contained in $\mathcal{S}_0$, thats mean that they also have a form
of \er{noninchgravortbstrjgghguittu2intrrrZZygjyghhjint}. Next
assume that a four-covector $(a_0,a_1,a_2,a_3)$ and a four-vector
$(b^0,b^1,b^2,b^3)$ on the group $\mathcal{S}_0$ are given. Then, by
inserting \er{noninchgravortbstrjgghguittu2intrrrZZygjyghhjint} into
\er{fgjfjhgghhgjgint} and \er{fgjfjhgghhgjghjhjint} in Section
\ref{CVFRM}
%
%
%
%
%
%
we obtained the following laws of transformation of four-covectors
and four-vectors on the group $\mathcal{S}_0$, i.e. under the change
of non-inertial cartesian coordinate systems:
\begin{equation}\label{fgjfjhgghhgjghjhjijhojpiiihjhjint}
\begin{cases}
a'_0=a_0-\sum_{k=1}^{3}\frac{1}{c}\left(\sum\limits_{j=1}^{3}\frac{dA_{kj}}{dt}\left(\frac{x^0}{c}\right)\,x_j+\frac{d
z_k}{dt}\left(\frac{x^0}{c}\right)\right)\left(\sum_{j=1}^{3}A_{kj}\left(\frac{x^0}{c}\right)\,a_j\right)
\\
a'_k=\sum_{j=1}^{3}A_{kj}\left(\frac{x^0}{c}\right)\,a_j\quad\quad\forall
k=1,2,3,
\end{cases}
\end{equation}
and
\begin{equation}\label{fgjfjhgghhgjgiuouoiuujkjkjkint}
\begin{cases}
b'^0=b^0
\\
b'^j=\frac{1}{c}\left(\sum\limits_{k=1}^{3}\frac{dA_{jk}}{dt}\left(\frac{x^0}{c}\right)\,x_k+\frac{d
z_j}{dt}\left(\frac{x^0}{c}\right)\right)\,b^0+\sum_{k=1}^{3}A_{jk}\left(\frac{x^0}{c}\right)\,b^k\quad\quad\forall
j=1,2,3.
\end{cases}
\end{equation}
%
%
%
%
%
%
Therefore, if we denote the four-vector $(b^0,b^1,b^2,b^3)$ and the
four-covector $(a_0,a_1,a_2,a_3)$ on the group $\mathcal{S}_0$ as:
\begin{equation}\label{fgjfjhgghhgjghjhjijhojihjhjjijhjjjint}
\begin{cases}
(b^0,b^1,b^2,b^3)=\left(\sigma,\frac{1}{c}\vec
b\right)\quad\text{where}\quad
\sigma:=b^0\;\;\text{and}\;\;\vec b:=c(b^1,b^2,b^3)\in\mathbb{R}^3,\\
(a_0,a_1,a_2,a_3)=(\psi,-\vec a)\quad\text{where}\quad
\psi:=a_0\;\;\text{and}\;\;\vec a:=-(a_1,a_2,a_3)\in\mathbb{R}^3,
\end{cases}
\end{equation}
then by \er{fgjfjhgghhgjghjhjijhojpiiihjhjint} and
\er{fgjfjhgghhgjgiuouoiuujkjkjkint} in the terms of time $t$ and
three-dimensional space $\vec x$, we obtain the following laws of
transformations of $\sigma$, $\vec b$, $\psi$ and $\vec a$ under the
change of non-inertial cartesian coordinate system:
\begin{equation}\label{fgjfjhgghhgjgiuouoiuujkjkjkojkoiuint}
\begin{cases}
\sigma'=\sigma
\\
\vec b'=A\left(t\right)\cdot\vec
b+\left(\frac{dA}{dt}\left(t\right)\cdot\vec x+\frac{d \vec
z}{dt}\left(t\right)\right)\sigma,
\end{cases}
\end{equation}
and
\begin{equation}\label{fgjfjhgghhgjghjhjijhojpiiihjhjuioujint}
\begin{cases}
\psi'=\psi+\frac{1}{c}\left(\frac{dA}{dt}\left(t\right)\cdot\vec
x+\frac{d \vec z}{dt}\left(t\right)\right)\cdot\left(A(t)\cdot\vec
a\right)
\\
\vec a'=A(t)\cdot\vec a.
\end{cases}
\end{equation}
In particular, if $\sigma:=b^0$ is the first coordinate of an
arbitrary four-vector $(b^0,b^1,b^2,b^3)$ on the group
$\mathcal{S}_0$, then $\sigma$ is a proper scalar field in the
frames of Definition \ref{bggghghgjint}. Moreover, if $\vec
a:=-(a_1,a_2,a_3)$, where $a_1,a_2,a_3$ are the last three
coordinates of an arbitrary four-covector $(a_0,a_1,a_2,a_3)$ on the
group $\mathcal{S}_0$, then $\vec a$ is a proper vector field in the
frames of Definition \ref{bggghghgjint}.

Next, since by Definition \ref{bggghghgjint} every three-dimensional
speed-like vector field, $\vec u$ transforms under the change of
non-inertial cartesian coordinate system as:
\begin{equation}
\label{NoIn3redPPN'int}\vec u'=A(t)\cdot \vec
u+\frac{dA}{dt}\left(t\right)\cdot\vec x+\frac{d \vec
z}{dt}\left(t\right),
\end{equation}
by comparing \er{NoIn3redPPN'int} with
\er{fgjfjhgghhgjgiuouoiuujkjkjkojkoiuint} we deduce that for every
speed-like vector field $\vec u$ the four-component field
$(u^0,u^1,u^2,u^3)$ defined by
\begin{equation}\label{fgjfjhgghhgjghjhjijhojihjhjjijhjjjjjuiiint}
(u^0,u^1,u^2,u^3):=\left(1,\frac{1}{c}\vec
u\right)\quad\text{where}\quad
u^0=1\;\;\text{and}\;\;(u^1,u^2,u^3)=\frac{1}{c}\vec
u\in\mathbb{R}^3,
\end{equation}
is a four-vector field on the group $\mathcal{S}_0$. We call such
four-vectors by the name vectors of type $1$. In particular, if
$\vec u$ is the velocity field, then the quantity defined by
\er{fgjfjhgghhgjghjhjijhojihjhjjijhjjjjjuiiint} is a a four-vector
field on the group $\mathcal{S}_0$ that we call the four-dimensional
speed.
%
%
%
%
%
%
Thus, in particular, if $\vec
r(t)=\left(r_1(t),r_2(t),r_3(t)\right)$ is a three-dimensional
trajectory of the motion of some particle, parameterized by the
global time $t$, then if we consider a curve
$\frac{1}{c}\left(ct,r_1(t),r_2(t),r_3(t)\right)$ in $\mathbb{R}^4$,
parameterized by the global time $t$, then the four-component field:
\begin{equation}\label{fgjfjhgghhgjghjhjijhojihjhjjijhjjjjjuiijhjhhint}
\left(1,\frac{1}{c}\frac{d\vec r}{dt}(t)\right):=
\left(1,\frac{1}{c}\frac{dr_1}{dt}(t),\frac{1}{c}\frac{dr_2}{dt}(t),\frac{1}{c}\frac{dr_3}{dt}(t)\right)
\end{equation}
is a four-vector field on the group $\mathcal{S}_0$.

 Similarly, if $\vec v$ is the vectorial gravitational
potential, then since $\vec v$ is a speed-like vector field, the
four-component field $(v^0,v^1,v^2,v^3)$ defined by
\begin{equation}\label{fgjfjhgghhgjghjhjijhojihjhjjijhjjjjjuiijjjkint}
(v^0,v^1,v^2,v^3):=\left(1,\frac{1}{c}\vec
v\right)\quad\text{where}\quad
v^0=1\;\;\text{and}\;\;(v^1,v^2,v^3)=\frac{1}{c}\vec v,
\end{equation}
is also a four-vector field on the group $\mathcal{S}_0$ that we
call the four-dimensional gravitational potential. In the same way,
if $\vec k$ is the vectorial potential of inertia, defined by
Definition \ref{jjhgyftdtdint}, then since $\vec k$ is a speed-like
vector field, the four-component field $(k^0,k^1,k^2,k^3)$ defined
by
\begin{equation}\label{fgjfjhgghhgjghjhjijhojihjhjjijhjjjjjuiijjjkjkkhhkkhint}
(k^0,k^1,k^2,k^3):=\left(1,\frac{1}{c}\vec
k\right)\quad\text{where}\quad
k^0=1\;\;\text{and}\;\;(k^1,k^2,k^3)=\frac{1}{c}\vec k,
\end{equation}
is also a four-vector field on the group $\mathcal{S}_0$ that we
call the four-dimensional potential of inertia.

Moreover, by \er{fgjfjhgghhgjgiuouoiuujkjkjkojkoiuint},
%
%
%
%
%
%
if we consider the field of four-dimensional moment of a particle
$(p^0,p^1,p^2,p^3)$ defined by
\begin{equation}\label{fgjfjhgghhgjghjhjijhojihjhjjijhjjjjjuiijjjklihhint}
(p^0,p^1,p^2,p^3):=\left(m,\frac{1}{c}(m\vec
u)\right)\quad\text{where}\quad
p^0=m\;\;\text{and}\;\;(p^1,p^2,p^3)=\frac{1}{c}(m\vec u),
\end{equation}
where $m$ is the mass of the particle and $\vec u$ is the velocity
of the particle, then $(p^0,p^1,p^2,p^3)$ is also a four-vector on
the group $\mathcal{S}_0$. Moreover, by comparing
\er{yuythfgfyftydtydtydtyddyyyhhddhhhredPPN111hgghjgintint}
with
\er{fgjfjhgghhgjgiuouoiuujkjkjkojkoiuint} we deduce that if we
consider the field of four-dimensional electric current
$(j^0,j^1,j^2,j^3)$ defined by
\begin{equation}\label{fgjfjhgghhgjghjhjijhojihjhjjijhjjjjjuiijjjklihhojjjoint}
(j^0,j^1,j^2,j^3):=\left(\rho,\frac{1}{c}\,\vec
j\right)\quad\text{where}\quad
j^0=\rho\;\;\text{and}\;\;(j^1,j^2,j^3)=\frac{1}{c}\,\vec j,
\end{equation}
where $\rho$ is the electric charge density and $\vec j$ is the
electric current density, then $(j^0,j^1,j^2,j^3)$ is also a
four-vector on the group $\mathcal{S}_0$.

On the other hand, for every proper three-dimensional vector field
$\vec G$ that satisfies due to Definition \ref{bggghghgjint}:
\begin{equation}
\label{NoIn1redPPN'int}\vec G'=A(t)\cdot\vec G,
\end{equation}
by comparing \er{NoIn1redPPN'int} with
\er{fgjfjhgghhgjgiuouoiuujkjkjkojkoiuint} we deduce that the
four-component field $(G^0,G^1,G^2,G^3)$ defined by
\begin{equation}\label{fgjfjhgghhgjghjhjijhojihjhjjijhjjjjjuiiklkllo;int}
(G^0,G^1,G^2,G^3):=\left(0,\vec G\right)\quad\text{where}\quad
G^0=0\;\;\text{and}\;\;(G^1,G^2,G^3)=\vec G,
\end{equation}
is also a four-vector field on the group $\mathcal{S}_0$.  We call
such four-vectors by the name vectors of type $0$.

Next, since by
\er{vhfffngghhjghhgPPNghghghutghfflklhjkjhjhjjgjkghhjint} the scalar
electromagnetic potential $\Psi$ and the vector electromagnetic
potential $\vec A$, under the change of non-inertial cartesian
coordinate system transform as:
\begin{equation}\label{vhfffngghhjghhgPPNghghghutghfflklhjkjhjhjjgjkghhjhhjhjkhjghlkkint}
\begin{cases}
\Psi'=
\Psi+\frac{1}{c}\left(\frac{dA}{dt}(t)\cdot\vec x+\frac{d\vec
z}{dt}(t)\right)\cdot\left(A(t)\cdot\vec A\right)
\\
\vec A'=A(t)\cdot \vec A,
\end{cases}
\end{equation}
by comparing
\er{vhfffngghhjghhgPPNghghghutghfflklhjkjhjhjjgjkghhjhhjhjkhjghlkkint}
with \er{fgjfjhgghhgjghjhjijhojpiiihjhjuioujint} we deduce that the
four-component field $(A_0,A_1,A_2,A_3)$ defined as
\begin{equation}\label{fgjfjhgghhgjghjhjijhojihjhjjijhjjjljljpkint}
(A_0,A_1,A_2,A_3)=(\Psi,-\vec A)\quad\text{where}\quad
A_0=\Psi\;\;\text{and}\;\;(A_1,A_2,A_3)=-\vec A,
\end{equation}
is a four-\underline{covector} field on the group $\mathcal{S}_0$.
We call this  four-covector field by the name four dimensional
electromagnetic potential. Next, since $(A_0,A_1,A_2,A_3)$ is a
four-covector field on the group $\mathcal{S}_0$, then it is well
known from the tensor analysis that the $16$-component field
$\{F_{ij}\}_{0\leq i,j\leq 3}$ defined in every non-inertial
cartesian coordinate system by
\begin{equation}\label{huohuioy89gjjhjffffff3478zzrrZZZhjhhjhhjjhhffGGhjjhint}
F_{ij}:=\frac{\partial A_j}{\partial x^i}-\frac{\partial
A_i}{\partial x^j}\quad\quad\forall\, i,j=0,1,2,3\,,
\end{equation}
is an antisymmetric two times covariant tensor field on the group
$\mathcal{S}_0$, which we call the covariant tensor of the
electromagnetic field. In particular, by inserting
\er{fgjfjhgghhgjghjhjijhojihjhjjijhjjjljljpkint} and
\er{fgjfjhgghint} into
\er{huohuioy89gjjhjffffff3478zzrrZZZhjhhjhhjjhhffGGhjjhint}
%
%
%
%
%
%
and denoting:
\begin{equation}\label{MaxVacFull1bjkgjhjhgjgjgkjfhjfdghghligioiuittrPPNkkkint}
\begin{cases}
(B_1,B_2,B_3)=\vec B:= curl_{\vec x} \vec A,\\
(E_1,E_2,E_3)=\vec E:=-\nabla_{\vec
x}\Psi-\frac{1}{c}\frac{\partial\vec A}{\partial t},
\end{cases}
\end{equation}
we deduce:
\begin{equation}\label{huohuioy89gjjhjffffff3478zzrrZZZhjhhjhhjjhhffGGhjjhiuuijkjjkint}
\begin{cases}
F_{00}=0\\ F_{0j}=-F_{j0}=E_j\quad\quad\forall\, j=1,2,3\\
F_{jj}=0\quad\quad\forall\, j=1,2,3
\\
F_{12}=-F_{21}=-B_3
\\
F_{13}=-F_{31}=B_2
\\
F_{23}=-F_{32}=-B_1\,.
\end{cases}
\end{equation}

Next assume that $T:=\left\{T_{ij}\right\}_{i,j=1,2,3}\in
\R^{3\times 3}$ is a $9$-component proper matrix valued field,
which, being a proper matrix field, by Definition \ref{bggghghgjint}
satisfies:
\begin{equation}\label{uguyytfddddgghjjghjjjihohjjkint}
T'=A(t)\cdot T\cdot A^T(t)=A(t)\cdot T\cdot
\left\{A(t)\right\}^{-1}.
\end{equation}
Next consider a $16$-component field $\{\mathcal{T}^{ij}\}_{0\leq
i,j\leq 3}$ defined in every non-inertial cartesian coordinate
system by
\begin{equation}\label{huohuioy89gjjhjffffff3478zzrrZZZhjhhjhhjjhhffGGhjjhuiiuihhjint}
\begin{cases}
\mathcal{T}^{00}=0
\\
\mathcal{T}^{0j}=\mathcal{T}^{j0}=0\quad\quad\forall\, j=1,2,3
\\
\mathcal{T}^{ij}:=T_{ij}\quad\quad\forall\, i,j=1,2,3\,,
\end{cases}
\end{equation}
Then, by inserting
\er{noninchgravortbstrjgghguittu2intrrrZZygjyghhjint} and
\er{uguyytfddddgghjjghjjjihohjjkint} into
\er{fgjfjhgghhgjghjhjkkkkggghint} in Section \ref{CVFRM} we  prove
that the field $\{\mathcal{T}^{ij}\}_{0\leq i,j\leq 3}$ defined by
\er{huohuioy89gjjhjffffff3478zzrrZZZhjhhjhhjjhhffGGhjjhuiiuihhjint}
is a two times \underline{contravariant} tensor field on the group
$\mathcal{S}_0$.
%
%
%
%
%
%

In particular, if we consider the $9$-component matrix field $I$
that defined in every cartesian coordinate system as
$I:=\left\{\delta_{ij}\right\}_{1,j=1,2,3}\in \R^{3\times 3}$, where
\begin{equation}\label{huohuioy89gjjhjffffff3478zzrrZZZhjhhjhhjjhhffGGhjjhuiiuihhjkklpklint}
\delta_{ij}:=
\begin{cases}
1\quad\text{if}\quad i=j
\\
0\quad\text{if}\quad i\neq j,
\end{cases}
\end{equation}
which is a proper matrix field, since
\begin{equation}\label{uguyytfddddgghjjghjjjihohjjkkkookint}
I=A(t)\cdot I\cdot \left\{A(t)\right\}^{-1},
\end{equation}
then the $16$-component field $\{{\Theta}^{ij}\}_{0\leq i,j\leq 3}$
defined in every non-inertial cartesian coordinate system by
\begin{equation}\label{huohuioy89gjjhjffffff3478zzrrZZZhjhhjhhjjhhffGGhjjhuiiuihhjkkoioiint}
\begin{cases}
{\Theta}^{00}=0
\\
{\Theta}^{0j}={\Theta}^{j0}=0\quad\quad\forall\, j=1,2,3
\\
{\Theta}^{ij}:=\delta_{ij}\quad\quad\forall\, i,j=1,2,3
\end{cases}
\end{equation}
is a two times contravariant tensor field on the group
$\mathcal{S}_0$ and moreover, this tensor is symmetric. We call
$\{{\Theta}^{ij}\}_{0\leq i,j\leq 3}$ the contravariant tensor of
the three-dimensional geometry.

Next, the scalar field $\tau:=\tau(x^0,x^1,x^2,x^3)$, defined in
every cartesian coordinate system as
\begin{equation}\label{uguyytfddddgghkjhhjuuiuiint}
\tau:=\frac{x^0}{c}=t,
\end{equation}
is a scalar on the group $\mathcal{S}_0$. Here $t$ is the global
non-relativistic time. Moreover,
by \er{uguyytfddddgghkjhhjuuiuiint}, the four-component field
$(v_0,v_1,v_2,v_3)$ defined as a gradient of the global time by:
\begin{equation}\label{fgjfjhgghhgjghjhjkkkkgjghghuiiiulkkjlkklplikklkjljghhgghkint}
v_0:=c\,\frac{\partial \tau}{\partial
x^0}(x^0,x^1,x^2,x^3)=1\quad\text{and}\quad v_j:=c\,\frac{\partial
\tau}{\partial x^j}(x^0,x^1,x^2,x^3)=0\quad\quad\forall\,j=1,2,3,
\end{equation}
is a \underline{four-covector} field on the group $\mathcal{S}_0$.

Finally, consider a motion of a classical particle with inertial
mass $m$, charge $\sigma$, place $\vec r(t)$ and velocity $\vec
u(t)=\frac{d\vec r}{dt}(t)$ in the outer gravitational field with
the vectorial gravitational potential $\vec v(\vec x,t)$, the outer
electromagnetic field with vectorial and scalar potentials $\vec
A(\vec x,t)$ and $\vec \Psi(\vec x,t)$, and additional conservative
field with scalar potential $V(\vec x,t)$ ruled by a Lagrangian
\er{vhfffngghkjgghfjjint}:
\begin{equation}\label{vhfffngghkjgghfjjuuuint}
L_0\left(\frac{d\vec r}{dt},\vec
r,t\right):=\frac{m}{2}\left|\frac{d\vec r}{dt}-\vec v(\vec
r,t)\right|^2-\sigma\left(\Psi(\vec r,t)-\frac{1}{c}\vec A(\vec
r,t)\cdot\frac{d\vec r}{dt}\right)+V(\vec r,t).
\end{equation}
Then $L_0$ is a scalar on the group $\mathcal{S}_0$. Moreover,
consider the generalized momentum of the particle $m$ by
\er{guytyurtydftyiujhint}:
\begin{equation}\label{guytyurtydftyiujhuuuint}
\vec P:=\nabla_{\vec r'}L_0\left(\vec r',\vec r,t\right)=m
\frac{d\vec r}{dt}-m\vec v(\vec r,t)+\frac{\sigma}{c}\vec A(\vec
r,t),
\end{equation}
consider a Hamiltonian
\begin{equation}\label{vhfffngghkjgghfjjhyjjfguuuint}
H_0\left(\vec P,\vec r,t\right):=\vec P\cdot\frac{d\vec
r}{dt}-L_0\left(\frac{d\vec r}{dt},\vec r,t\right),
\end{equation}
which by \er{vhfffngghkjgghfjjghghghint} satisfies
\begin{equation}\label{vhfffngghkjgghfjjghghghuuuint} H_0\left(\vec
P,\vec r,t\right)= \vec P\cdot\vec v(\vec r,t)+
\frac{1}{2m}\left|\vec P-\frac{\sigma}{c}\vec A(\vec
r,t)\right|^2+\sigma\left(\Psi(\vec r,t)-\frac{1}{c}\vec A(\vec
r,t)\cdot\vec v(\vec r,t)\right)-V\left(\vec r,t\right),
\end{equation}
and furthermore, define the four-dimensional generalized momentum
$(P_0,P_1,P_2,P_3)$ as:
\begin{equation}\label{fgjfjhgghhgjghjhjijhojihjhjjijhjjjljljpkhkhjgjgint}
(P_0,P_1,P_2,P_3):=\left(\frac{1}{c}H_0,-\vec
P\right)\quad\text{where}\quad
P_0=\frac{1}{c}H_0\;\;\text{and}\;\;(P_1,P_2,P_3)=-\vec P,
\end{equation}
Then, since by \er{vhfffngghkjgghfjjghghghuuuint} and
\er{guytyurtydftyiujhuuuint}, under the change of non-inertial
cartesian coordinate system $H_0$ and $\vec P$ transform as
\begin{equation}\label{vhfffngghhjghhgPPNghghghutghfflklhjkjhjhjjgjkghhjhhjhjkhjghlkkllint}
\begin{cases}
H_0'=
H_0+\left(\frac{dA}{dt}(t)\cdot\vec x+\frac{d\vec
z}{dt}(t)\right)\cdot\left(A(t)\cdot\vec P\right)
\\
\vec P'=A(t)\cdot \vec P,
\end{cases}
\end{equation}
by comparing
\er{vhfffngghhjghhgPPNghghghutghfflklhjkjhjhjjgjkghhjhhjhjkhjghlkkllint}
with \er{fgjfjhgghhgjghjhjijhojpiiihjhjuioujint} we deduce that the
four-dimensional momentum $(P_0,P_1,P_2,P_3)$ is a
four-\underline{covector} on the group $\mathcal{S}_0$.

\subsubsection{Pseudo-metric tensors of the four-dimensional
space-time} Consider $\{g^{ij}\}_{0\leq i,j\leq 3}$ to be a two
times contravariant tensor field on the group $\mathcal{S}_0$,
defined by
\begin{equation}\label{hoyuiouigyfghgjh3478zzrrZZffGGjkkjojjint}
g^{ij}:=v^iv^j-{\Theta}^{ij}\quad\quad\forall\,i,j=0,1,2,3,
\end{equation}
where $\{{\Theta}^{ij}\}_{0\leq i,j\leq 3}$ is the contravariant
tensor of the three-dimensional geometry, defined by
\er{huohuioy89gjjhjffffff3478zzrrZZZhjhhjhhjjhhffGGhjjhuiiuihhjkkoioiint}
and being a two times contravariant tensor, and $(v^0,v^1,v^2,v^3)$
is the four-dimensional gravitational potential, defined by
\er{fgjfjhgghhgjghjhjijhojihjhjjijhjjjjjuiijjjkint} and being a
four-vector. Then,
in Section \ref{CVFRM} we obtain that $\{g^{ij}\}_{0\leq i,j\leq 3}$
is indeed a two times contravariant tensor field on the group
$\mathcal{S}_0$ and moreover, this tensor is symmetric. Moreover, by
\er{huohuioy89gjjhjffffff3478zzrrZZZhjhhjhhjjhhffGGhjjhuiiuihhjkkoioiint}
and \er{fgjfjhgghhgjghjhjijhojihjhjjijhjjjjjuiijjjkint} we have:
\begin{equation}\label{hoyuiouigyfghgjh3478zzrrZZffGGjkkjint}
\begin{cases}
g^{00}=1\\
g^{ij}=-\delta_{ij}+\frac{v^iv^j}{c^2}\quad\forall 1\leq i,j\leq 3\\
g^{0j}=g^{j0}=\frac{v^j}{c}\quad\forall 1\leq j\leq 3,
\end{cases}
\end{equation}
where $\vec v=(v^1,v^2,v^3)$ is the three-dimensional vectorial
gravitational potential. We call the tensor $\{g^{ij}\}_{0\leq
i,j\leq 3}$ the contravariant pseudo-metric tensor of the
four-dimensional space-time. Next consider a $16$-component field
$\{g_{ij}\}_{0\leq i,j\leq 3}$ defined by
\begin{equation}\label{hoyuiouigyfg3478zzrrZZffGGhhjhjint}
\begin{cases}
g_{00}=1-\frac{|\vec v|^2}{c^2}\\
g_{ij}=-\delta_{ij}\quad\forall 1\leq i,j\leq 3\\
g_{0j}=g_{j0}=\frac{v^j}{c}\quad\forall 1\leq j\leq 3.
\end{cases}
\end{equation}
Then in Section \ref{CVFRM}
%
%
%
%
%
%
we deduce:
\begin{equation}\label{fgjfjhgghhgjghjhjkkkkgjghghuiiiuujhjhint}
\sum_{k=0}^{3}g^{ik}g_{kj}=\begin{cases}
1\quad\text{if}\quad i=j\\
0\quad\text{if}\quad i\neq j
\end{cases}\quad\quad\forall\, i,j=0,1,2,3.
\end{equation}
Therefore,
we obtain that $\{g_{ij}\}_{i,j=0,1,2,3}$ is a two times covariant
tensor on the group $\mathcal{S}_0$, and moreover, this tensor is
symmetric. We call the tensor $\{g_{ij}\}_{0\leq i,j\leq 3}$
covariant pseudo-metric tensor of the four-dimensional space-time.
Using \er{fgjfjhgghhgjghjhjkkkkgjghghuiiiuujhjhint} we also obtain
that the pseudo-metric tensors $\{g_{ij}\}_{i,j=0,1,2,3}$ and
$\{g^{ij}\}_{0\leq i,j\leq 3}$ are non-degenerate.
Moreover, it can be easily calculated that if
we consider the $4\times 4$-matrix:
\begin{equation}\label{fgjfjhgghhgjghjhjkkkkgjghghuiiiuujhjhjkljint}
G=\{g_{ij}\}_{0\leq i,j\leq 3},
\end{equation}
then
\begin{equation}\label{fgjfjhgghhgjghjhjkkkkgjghghuiiiuujhjh1int}
\text{det}\,G=-1.
\end{equation}

Thus, with the covariant and contravariant pseudo-metric tensors we
can lower and lift indexes of arbitrary tensors. In particular given
a four-covector $(a_0,a_1,a_2,a_3)$ and a four-vector
$(b^0,b^1,b^2,b^3)$ on the group $\mathcal{S}_0$ we can define the
corresponding lifted four-vector $(a^0,a^1,a^2,a^3)$ and the
corresponded lowered four-covector $(b_0,b_1,b_2,b_3)$ by
\begin{equation}\label{fgjfjhgghhgjghjhjkkkkgjghghuiiiulkkjKKint}
(a^0,a^1,a^2,a^3):=\Big\{\sum_{k=0}^{3}g^{mk}a_{k}\Big\}_{m=0,1,2,3}\quad\text{and}\quad
(b_0,b_1,b_2,b_3):=\Big\{\sum_{k=0}^{3}g_{mk}b^{k}\Big\}_{m=0,1,2,3}
\end{equation}
Then by \er{hoyuiouigyfghgjh3478zzrrZZffGGjkkjint},
\er{hoyuiouigyfg3478zzrrZZffGGhhjhjint} and
\er{fgjfjhgghhgjghjhjkkkkgjghghuiiiulkkjKKint} we have:
%
%
%
%
%
%
\begin{equation}\label{fgjfjhgghhgjghjhjkkkkgjghghuiiiulkkjKKyuyyuiouiint}
a^0=a_{0}+\sum_{k=1}^{3}\frac{1}{c}v^k a_{k}\quad\text{and}\quad
a^m=-a_{m}+\frac{1}{c}a^0v^m\quad\quad\forall m=1,2,3,
\end{equation}
and
\begin{equation}\label{fgjfjhgghhgjghjhjkkkkgjghghuiiiulkkjKKyuyyuuuuiiuiiint}
b_0=b^0-\sum_{k=1}^{3}\frac{1}{c}v^kb_k \quad\text{and}\quad
b_m=-b^{m}+\frac{1}{c} b^{0}v^m\quad\quad\forall m=1,2,3.
\end{equation}
In particular, we have:
%
%
%
%
%
%
\begin{equation}\label{fgjfjhgghhgjghjhjkkkkgjghghuiiiulkkjKKyuyyu0ioioiogghghghint}
b^0a_0+\sum_{k=1}^{3}b^ka_k=b^0a^0-\sum_{k=1}^{3}b_{k}a_k.
\end{equation}

Next, if for every speed-like vector field $\vec u$ we consider the
four-vector field $(u^0,u^1,u^2,u^3)$ defined by
\er{fgjfjhgghhgjghjhjijhojihjhjjijhjjjjjuiiint} as:
\begin{equation}\label{fgjfjhgghhgjghjhjijhojihjhjjijhjjjjjuiikkjint}
(u^0,u^1,u^2,u^3):=\left(1,\frac{1}{c}\vec
u\right)\quad\text{where}\quad
u^0=1\;\;\text{and}\;\;(u^1,u^2,u^3)=\frac{1}{c}\vec
u\in\mathbb{R}^3,
\end{equation}
then, by \er{fgjfjhgghhgjghjhjkkkkgjghghuiiiulkkjKKyuyyuuuuiiuiiint}
the corresponding lowered four-covector field $(u_0,u_1,u_2,u_3)$
satisfies:
\begin{multline}\label{fgjfjhgghhgjghjhjijhojihjhjjijhjjjjjuiikkjjnjint}
(u_0,u_1,u_2,u_3):=\left(1+\frac{1}{c^2}\left(\vec u-\vec
v\right)\cdot\vec v,-\frac{1}{c}\left(\vec u-\vec
v\right)\right)\quad\text{where}\quad\\
u_0=1+\frac{1}{c^2}\left(\vec u-\vec v\right)\cdot\vec
v\;\;\text{and}\;\;(u_1,u_2,u_3)=-\frac{1}{c}\left(\vec u-\vec
v\right)\in\mathbb{R}^3.
\end{multline}
Moreover, in the case where $(u^0,u^1,u^2,u^3)$  is a
four-dimensional speed, we call the corresponding lowered
four-covector field $(u_0,u_1,u_2,u_3)$ by the name four-dimensional
cospeed. In particular, if we consider the four-dimensional
gravitational potential $(v^0,v^1,v^2,v^3)$ defined by
\er{fgjfjhgghhgjghjhjijhojihjhjjijhjjjjjuiijjjkint}:
\begin{equation}\label{fgjfjhgghhgjghjhjijhojihjhjjijhjjjjjuiijjjkhjhjhjint}
(v^0,v^1,v^2,v^3):=\left(1,\frac{1}{c}\vec v\right),
\end{equation}
then by \er{fgjfjhgghhgjghjhjijhojihjhjjijhjjjjjuiikkjjnjint} we
obtain that the corresponding lowered four-covector field
$(v_0,v_1,v_2,v_3)$, that we call the four-covector of gravitational
potential, satisfies:
\begin{equation}\label{fgjfjhgghhgjghjhjijhojihjhjjijhjjjjjuiikkjjnj1int}
(v_0,v_1,v_2,v_3):=\left(1,0,0,0\right).
\end{equation}
Note that the four-covector of gravitational potential, defined by
\er{fgjfjhgghhgjghjhjijhojihjhjjijhjjjjjuiikkjjnj1int} coincides
with the four-covector defined by
\er{fgjfjhgghhgjghjhjkkkkgjghghuiiiulkkjlkklplikklkjljghhgghkint} as
the gradient of the scalar of global time. Next, by
\er{fgjfjhgghhgjghjhjijhojihjhjjijhjjjjjuiijjjkhjhjhjint} and
\er{fgjfjhgghhgjghjhjijhojihjhjjijhjjjjjuiikkjjnj1int} we clearly
have:
\begin{equation}\label{fgjfjhgghhgjghjhjijhojihjhjjijhjjjjjuiijjjkhjhjhjuiiuuuyuint}
c^2\left(\sum_{j=0}^{3}\sum_{k=0}^{3}\,g^{jk}\frac{\partial\tau}{\partial
x^j}\frac{\partial\tau}{\partial
x^k}\right)=\sum_{j=0}^{3}\sum_{k=0}^{3}g^{jk}v_jv_k=\sum_{j=0}^{3}\sum_{k=0}^{3}g_{jk}v^jv^k=\sum_{j=0}^{3}v^jv_j=1,
\end{equation}
where $\tau$ is the scalar of the global time on the group
$\mathcal{S}_0$, defined by \er{uguyytfddddgghkjhhjuuiuiint}.
Finally, we clearly have
\begin{equation}\label{fgjfjhfgfgfgfhint}
\sum_{k=0}^{3}{\Theta}^{mk}\frac{\partial\tau}{\partial
x^k}=\sum_{k=0}^{3}{\Theta}^{mk}v_k=0\quad\quad\forall\, m=0,1,2,3,
\end{equation}
where $\Theta^{ij}$ is the contravariant tensor of the
three-dimensional geometry, defined by
\er{huohuioy89gjjhjffffff3478zzrrZZZhjhhjhhjjhhffGGhjjhuiiuihhjkkoioiint}.

%
%
%
%
%
%
Moreover, if we consider the field of four-vector of the moment of a
particle $(p^0,p^1,p^2,p^3)$ defined by
\er{fgjfjhgghhgjghjhjijhojihjhjjijhjjjjjuiijjjklihhint} as
\begin{equation}\label{fgjfjhgghhgjghjhjijhojihjhjjijhjjjjjuiijjjklihhjjjint}
(p^0,p^1,p^2,p^3):=\left(m,\frac{1}{c}(m\vec
u)\right)\quad\text{where}\quad
p^0=m\;\;\text{and}\;\;(p^1,p^2,p^3)=\frac{1}{c}(m\vec u),
\end{equation}
where $m$ is the mass of the particle and $\vec u$ is the velocity
of the particle, then the corresponding lowered four-covector field
$(p_0,p_1,p_2,p_3)$, which we call the four-covector of momentum,
satisfies:
\begin{multline}\label{fgjfjhgghhgjghjhjijhojihjhjjijhjjjjjuiikkjjnjjbjhjhint}
(p_0,p_1,p_2,p_3):=\left(m\left(1+\frac{1}{c^2}\left(\vec u-\vec
v\right)\cdot\vec v\right),-\frac{m}{c}\left(\vec u-\vec
v\right)\right)\quad\text{where}\quad\\
p_0=m\left(1+\frac{1}{c^2}\left(\vec u-\vec v\right)\cdot\vec
v\right)\;\;\text{and}\;\;(p_1,p_2,p_3)=-\frac{m}{c}\left(\vec
u-\vec v\right).
\end{multline}
In particular, by
\er{fgjfjhgghhgjghjhjkkkkgjghghuiiiulkkjKKyuyyu0ioioiogghghghint} we
have:
\begin{equation}\label{fgjfjhgghhgjghjhjkkkkgjghghuiiiulkkjKKyuyyu0ioioiogghghghgghghint}
-\frac{c^2}{2m}\left(p^0p_0+\sum_{k=1}^{3}p^kp_k\right)=\frac{mc^2}{2}\left(\frac{1}{c^2}\left|\vec
u-\vec v\right|^2-1\right)=\frac{m}{2}\left|\vec u-\vec
v\right|^2-\frac{mc^2}{2}.
\end{equation}
Moreover, if we consider the  four-dimensional electric current
$(j^0,j^1,j^2,j^3)$ defined by
\er{fgjfjhgghhgjghjhjijhojihjhjjijhjjjjjuiijjjklihhojjjoint} as
\begin{equation}\label{fgjfjhgghhgjghjhjijhojihjhjjijhjjjjjuiijjjklihhojjjokjkjint}
(j^0,j^1,j^2,j^3):=\left(\rho,\frac{1}{c}\,\vec
j\right)\quad\text{where}\quad
j^0=\rho\;\;\text{and}\;\;(j^1,j^2,j^3)=\frac{1}{c}\,\vec j,
\end{equation}
where $\rho$ is the electric charge density and $\vec j$ is the
electric current density, then the corresponding lowered
four-covector field $(j_0,j_1,j_2,j_3)$, which we call the
four-covector of current, satisfies:
\begin{multline}\label{fgjfjhgghhgjghjhjijhojihjhjjijhjjjjjuiikkjjnjjbjhjhiyuyugint}
(j_0,j_1,j_2,j_3):=\left(\rho+\frac{1}{c^2}\left(\vec j-\rho\vec
v\right)\cdot\vec v,-\frac{1}{c}\left(\vec j-\rho\vec
v\right)\right)\quad\text{where}\quad\\
j_0=\rho+\frac{1}{c^2}\left(\vec j-\rho\vec v\right)\cdot\vec
v\;\;\text{and}\;\;(j_1,j_2,j_3)=-\frac{1}{c}\left(\vec j-\rho\vec
v\right).
\end{multline}
Finally, if $\Psi$ is the scalar electromagnetic potential and $\vec
A$ is the vector electromagnetic potential and we consider the
four-\underline{covector} field of four dimensional electromagnetic
potential $(A_0,A_1,A_2,A_3)$, defined by
\er{fgjfjhgghhgjghjhjijhojihjhjjijhjjjljljpkint} as:
\begin{equation}\label{fgjfjhgghhgjghjhjijhojihjhjjijhjjjljljpkikhjint}
(A_0,A_1,A_2,A_3)=(\Psi,-\vec A)\quad\text{where}\quad
A_0=\Psi\;\;\text{and}\;\;(A_1,A_2,A_3)=-\vec A,
\end{equation}
then by inserting
\er{fgjfjhgghhgjghjhjijhojihjhjjijhjjjljljpkikhjint} into
\er{fgjfjhgghhgjghjhjkkkkgjghghuiiiulkkjKKyuyyuiouiint} we deduce
that the corresponding lifted four-vector field $(A^0,A^1,A^2,A^3)$,
which we call the four-vector of electromagnetic potential,
satisfies:
\begin{equation}\label{fgjfjhgghhgjghjhjkkkkgjghghuiiiulkkjKKyuyyuiouigjguuuiuiint}
A^0=\Psi-\frac{1}{c}\vec v\cdot\vec
A\quad\text{and}\quad(A^1,A^2,A^3)=\vec
A+\frac{1}{c}\left(\Psi-\frac{1}{c}\vec v\cdot\vec A\right)\vec v.
\end{equation}
On the other hand, the proper scalar electromagnetic potential
$\Psi_0$ was defined by
\er{vhfffngghhjghhgPPNghghghutghffugghjhjkjjklint} as:
\begin{equation}\label{vhfffngghhjghhgPPNghghghutghffugghjhjkjjklhjghint}
\Psi_0:=\Psi-\frac{1}{c}\vec A\cdot\vec v.
\end{equation}
Thus we rewrite
\er{fgjfjhgghhgjghjhjkkkkgjghghuiiiulkkjKKyuyyuiouigjguuuiuiint} as:
\begin{equation}\label{fgjfjhgghhgjghjhjkkkkgjghghuiiiulkkjKKyuyyuiouigjguuuiui1int}
A^0=\Psi_0\quad\text{and}\quad(A^1,A^2,A^3)=\vec
A+\frac{1}{c}\Psi_0\vec v.
\end{equation}

Next given a two times covariant tensor $\{c_{mn}\}_{m,n=0,1,2,3}$
on the group $\mathcal{S}_0$
we consider two times contravariant lifted tensor
on $\mathcal{S}_0$: $\{c^{mn}\}_{m,n=0,1,2,3}$
defined by:
\begin{equation}\label{fgjfjhgghhgjghjhjkkkkgjghghuiiiulkkjlkklKKint}
c^{mn}:=\sum_{k=0}^{3}\sum_{j=0}^{3}g^{mj}g^{nk}c_{jk}
\quad\quad\forall\, m,n=0,1,2,3.
\end{equation}
In particular, if $\{F_{ij}\}_{0\leq i,j\leq 3}$ is the
antisymmetric two times covariant tensor field of the
electromagnetic field on the group $\mathcal{S}_0$, which by
\er{huohuioy89gjjhjffffff3478zzrrZZZhjhhjhhjjhhffGGhjjhiuuijkjjkint}
satisfies:
\begin{equation}\label{huohuioy89gjjhjffffff3478zzrrZZZhjhhjhhjjhhffGGhjjhiuuijkjjkihjhjhint}
\begin{cases}
F_{00}=0\\ F_{0j}=-F_{j0}=E_j\quad\quad\forall\, j=1,2,3\\
F_{jj}=0\quad\quad\forall\, j=1,2,3
\\
F_{12}=-F_{21}=-B_3
\\
F_{13}=-F_{31}=B_2
\\
F_{23}=-F_{32}=-B_1\,,
\end{cases}
\end{equation}
then by inserting
\er{huohuioy89gjjhjffffff3478zzrrZZZhjhhjhhjjhhffGGhjjhiuuijkjjkihjhjhint}
into \er{fgjfjhgghhgjghjhjkkkkgjghghuiiiulkkjlkklKKint}, using
\er{hoyuiouigyfghgjh3478zzrrZZffGGjkkjint} and denoting:
\begin{equation}\label{MaxVacFullPPNhjjghint}
\begin{cases}
(D_1,D_2,D_3)=\vec D:=\vec E+\frac{1}{c}\,\vec v\times
\vec B\\
(H_1,H_2,H_3)=\vec H:=\vec B+\frac{1}{c}\,\vec v\times \vec D,
\end{cases}
\end{equation}
we deduce:
\begin{equation}\label{khjhhkfgjfjhgghhgjghjhjkkkkgjghghuiiiulkkjlkklKKgfgjhjjghgjhhjhjhhhhhghhgint}
\begin{cases}
F^{00}=0
\\
F^{0j}=-F^{j0}=-D_j\quad\forall\, j=1,2,3,\\
F^{jj}=0 \quad\forall\, j=1,2,3,\\
F^{12}=-F^{21}=-H_3\\
F^{13}=-F^{31}=H_2\\
F^{23}=-F^{32}=-H_1.
\end{cases}
\end{equation}
In particular, by
\er{huohuioy89gjjhjffffff3478zzrrZZZhjhhjhhjjhhffGGhjjhiuuijkjjkihjhjhint}
and
\er{khjhhkfgjfjhgghhgjghjhjkkkkgjghghuiiiulkkjlkklKKgfgjhjjghgjhhjhjhhhhhghhgint},
using \er{MaxVacFullPPNhjjghint} we deduce that:
%
%
%
%
%
%
\begin{equation}\label{MaxVacFullPPNhjjghjjkjhh1int}
-\sum_{j=0}^{3}\sum_{k=0}^{3}\frac{1}{4}F^{jk}F_{jk}=\frac{1}{2}|\vec
D|^2-\frac{1}{2}|\vec B|^2.
\end{equation}

\subsubsection{Maxwell equations in covariant formulation}
%
%
%
%
%
%
In Section \ref{CVFRM} we prove that, since the lifted contravariant
tensor of the electromagnetic field $\{F^{ij}\}_{0\leq i,j\leq 3}$
on the group $\mathcal{S}_0$, considered in
\er{khjhhkfgjfjhgghhgjghjhjkkkkgjghghuiiiulkkjlkklKKgfgjhjjghgjhhjhjhhhhhghhgint}
is antisymmetric, then
the following four-component field:
\begin{equation}\label{MaxVacFullPPNhjjghjjkjhhoujiiikjjihjhiuiuinthhint}
\left\{\sum_{j=0}^{3}\frac{\partial F^{kj}}{\partial
x^j}+\sum_{j=0}^{3}\frac{F^{kj}}{\sqrt{|\text{det}\,G|}}\frac{\partial}{\partial
x^j}\left(\sqrt{|\text{det}\,G|}\right)\right\}_{0\leq k\leq 3}
\end{equation}
is a four-vector on the group $\mathcal{S}_0$, where the $4\times
4$-matrix $G$ is defined as $G:=\{g_{ij}\}_{0\leq i,j\leq 3}$.
Then, since the matrix $G$ satisfies $\text{det}\,G=-1$ in every
cartesian coordinate system,
then
\begin{equation}\label{MaxVacFullPPNhjjghjjkjhhoujiiikjjihjhiuiuint}
\left\{\sum_{j=0}^{3}\frac{\partial F^{kj}}{\partial
x^j}+\sum_{j=0}^{3}\frac{F^{kj}}{\sqrt{|\text{det}\,G|}}\frac{\partial}{\partial
x^j}\left(\sqrt{|\text{det}\,G|}\right)\right\}_{0\leq k\leq
3}=\left\{\sum_{j=0}^{3}\frac{\partial F^{kj}}{\partial
x^j}\right\}_{0\leq k\leq 3}.
\end{equation}
Note here that we denoted the matrix $G=\{g_{ij}\}_{0\leq i,j\leq
3}$ by the same letter as the Gravitational Constant $G$. However,
there is no ambiguity, since in the second case $G$ is a constant
scalar and in the first case $G$ is a matrix. Moreover, we will use
the matrix notation $G=\{g_{ij}\}_{0\leq i,j\leq 3}$ only in the
expressions containing term $\text{det}\,G$. Then, by
\er{khjhhkfgjfjhgghhgjghjhjkkkkgjghghuiiiulkkjlkklKKgfgjhjjghgjhhjhjhhhhhghhgint},
denoting $(x^0,x^1,x^2,x^3):=(ct,x_1,x_2,x_3)=(ct,\vec x)$, we
deduce:
%
%
%
%
%
%
\begin{equation}\label{khjhhkfgjfjhgghhgjghjhjkkkkgjghghuiiiulkkjlkklKKgfgjhjjghgjhhjhjhhhhhghhgtyytojjjjjint}
\left(\sum_{j=0}^{3}\frac{\partial F^{0j}}{\partial
x^j},\sum_{j=0}^{3}\frac{\partial F^{1j}}{\partial
x^j},\sum_{j=0}^{3}\frac{\partial F^{2j}}{\partial
x^j},\sum_{j=0}^{3}\frac{\partial F^{3j}}{\partial
x^j}\right)=\left(-div_{\vec x}\vec D,
\left(\frac{1}{c}\frac{\partial \vec D}{\partial t}-curl_{\vec
x}\vec H\right)\right).
\end{equation}
Therefore, by
\er{khjhhkfgjfjhgghhgjghjhjkkkkgjghghuiiiulkkjlkklKKgfgjhjjghgjhhjhjhhhhhghhgtyytojjjjjint},
the first pair of Maxwell Equations in
\er{MaxVacFull1bjkgjhjhgjaaaint}:
\begin{equation}\label{MaxVacFullPPNhjjghyghghiyyhhhjint}
\begin{cases}
curl_{\vec x} \vec H=\frac{4\pi}{c}\vec j+\frac{1}{c}\frac{\partial \vec D}{\partial t}\\
div_{\vec x} \vec D= 4\pi\rho,
\end{cases}
\end{equation}
is equivalent to the following equations:
\begin{equation}\label{khjhhkfgjfjhgghhgjghjhjkkkkgjghghuiiiulkkjlkklKKgfgjhjjghgjhhjhjhhhhhghhgtyytojjjjjuyiyint}
\left(\sum_{j=0}^{3}\frac{\partial F^{0j}}{\partial
x^j},\sum_{j=0}^{3}\frac{\partial F^{1j}}{\partial
x^j},\sum_{j=0}^{3}\frac{\partial F^{2j}}{\partial
x^j},\sum_{j=0}^{3}\frac{\partial F^{3j}}{\partial x^j}\right)
=-4\pi(j^0,j^1,j^2,j^3),
\end{equation}
where $(j^0,j^1,j^2,j^3)$ is the four-vector of electric current on
the group $\mathcal{S}_0$ defined by
\er{fgjfjhgghhgjghjhjijhojihjhjjijhjjjjjuiijjjklihhojjjoint} as:
\begin{equation}
\label{fgjfjhgghhgjghjhjijhojihjhjjijhjjjjjuiijjjklihhojjjoouuoiuiuint}
(j^0,j^1,j^2,j^3):=\left(\rho,\frac{1}{c}\,\vec j\right)
\end{equation}
Note that in both sides of equation
\er{khjhhkfgjfjhgghhgjghjhjkkkkgjghghuiiiulkkjlkklKKgfgjhjjghgjhhjhjhhhhhghhgtyytojjjjjuyiyint}
we have four-vectors and thus
\er{khjhhkfgjfjhgghhgjghjhjkkkkgjghghuiiiulkkjlkklKKgfgjhjjghgjhhjhjhhhhhghhgtyytojjjjjuyiyint}
is a covariant form of \er{MaxVacFullPPNhjjghyghghiyyhhhjint}. On
the other hand, the second pair of Maxwell Equations in
\er{MaxVacFull1bjkgjhjhgjaaaint}:
\begin{equation}\label{MaxVacFullPPNhjjghyghghiyyint}
\begin{cases}
curl_{\vec x} \vec E+\frac{1}{c}\frac{\partial \vec B}{\partial t}=0\\
div_{\vec x} \vec B=0,
\end{cases}
\end{equation}
is equivalent to
\er{MaxVacFull1bjkgjhjhgjgjgkjfhjfdghghligioiuittrPPNkkkint}, i.e.
to the following:
\begin{equation}\label{MaxVacFull1bjkgjhjhgjgjgkjfhjfdghghligioiuittrPPNkkkouiuiuuiuiint}
\begin{cases}
\vec B= curl_{\vec x} \vec A,\\
\vec E=-\nabla_{\vec x}\Psi-\frac{1}{c}\frac{\partial\vec
A}{\partial t},
\end{cases}
\end{equation}
On the other hand, as before, by
\er{huohuioy89gjjhjffffff3478zzrrZZZhjhhjhhjjhhffGGhjjhiuuijkjjkint}
we can rewrite
\er{MaxVacFull1bjkgjhjhgjgjgkjfhjfdghghligioiuittrPPNkkkouiuiuuiuiint}
in the form of
\er{huohuioy89gjjhjffffff3478zzrrZZZhjhhjhhjjhhffGGhjjhint}:
\begin{equation}\label{huohuioy89gjjhjffffff3478zzrrZZZhjhhjhhjjhhffGGhjjhuyuyuint}
F_{ij}=\frac{\partial A_j}{\partial x^i}-\frac{\partial
A_i}{\partial x^j}\quad\quad\forall\, i,j=0,1,2,3\,,
\end{equation}
where $(A_0,A_1,A_2,A_3)$ is the four-covector of the
electromagnetic potential on the group $\mathcal{S}_0$ defined by
\er{fgjfjhgghhgjghjhjijhojihjhjjijhjjjljljpkint} as:
\begin{equation}\label{fgjfjhgghhgjghjhjijhojihjhjjijhjjjljljpkyuuyyuint}
(A_0,A_1,A_2,A_3)=(\Psi,-\vec A).
\end{equation}
Note that in both sides of equation
\er{huohuioy89gjjhjffffff3478zzrrZZZhjhhjhhjjhhffGGhjjhuyuyuint} we
have two time covariant tensors, and thus
\er{huohuioy89gjjhjffffff3478zzrrZZZhjhhjhhjjhhffGGhjjhuyuyuint} is
a covariant form of \er{MaxVacFullPPNhjjghyghghiyyint}. Finally,
the relations between $(\vec E,\vec B)$ and $(\vec D,\vec H)$ in
\er{MaxVacFull1bjkgjhjhgjaaaint}:
\begin{equation}\label{MaxVacFullPPNhjjghojjkjkiouint}
\begin{cases}
\vec D=\vec E+\frac{1}{c}\,\vec v\times
\vec B\\
\vec H=\vec B+\frac{1}{c}\,\vec v\times \vec D,
\end{cases}
\end{equation}
are equivalent to the following covariant equations:
\begin{equation}\label{fgjfjhgghhgjghjhjkkkkgjghghuiiiulkkjlkklKKkjkjint}
F^{mn}:=\sum_{k=0}^{3}\sum_{j=0}^{3}g^{mj}g^{nk}F_{jk}
\quad\quad\forall\, m,n=0,1,2,3.
\end{equation}
Thus by
\er{huohuioy89gjjhjffffff3478zzrrZZZhjhhjhhjjhhffGGhjjhuyuyuint},
\er{fgjfjhgghhgjghjhjkkkkgjghghuiiiulkkjlkklKKkjkjint} and
\er{khjhhkfgjfjhgghhgjghjhjkkkkgjghghuiiiulkkjlkklKKgfgjhjjghgjhhjhjhhhhhghhgtyytojjjjjuyiyint}
together, we deduce that the full system of Maxwell Equations in
\er{MaxVacFull1bjkgjhjhgjaaaint}:
\begin{equation}\label{MaxVacFullPPNhjjghyghghiyyhhint}
\begin{cases}
curl_{\vec x} \vec H=\frac{4\pi}{c}\vec j+\frac{1}{c}\frac{\partial \vec D}{\partial t}\\
div_{\vec x} \vec D=4\pi\rho\\
curl_{\vec x} \vec E+\frac{1}{c}\frac{\partial \vec B}{\partial t}=0\\
div_{\vec x} \vec B=
0\\
\vec E=\vec D-\frac{1}{c}\,\vec v\times
\vec B\\
\vec H=\vec B+\frac{1}{c}\,\vec v\times \vec D,
\end{cases}
\end{equation}
is equivalent to the following covariant equations:
\begin{equation}\label{khjhhkfgjfjhgghhgjghjhjkkkkgjghghuiiiulkkjlkklKKgfgjhjjghgjhhjhjhhhhhghhgtyytojjjjjuyiyuuint}
\sum_{j=0}^{3}\frac{\partial}{\partial
x^j}\left(\sum_{m=0}^{3}\sum_{n=0}^{3}g^{km}g^{jn}\left(\frac{\partial
A_n}{\partial x^m}-\frac{\partial A_m}{\partial
x^n}\right)\right)=-4\pi j^k\quad\quad\forall\, k=0,1,2,3.
\end{equation}
Note that equations
\er{khjhhkfgjfjhgghhgjghjhjkkkkgjghghuiiiulkkjlkklKKgfgjhjjghgjhhjhjhhhhhghhgtyytojjjjjuyiyuuint}
are fully analogous to the covariant formulation of Maxwell
equations in Special Relativity and the only difference is the
choice of the pseudo-metric tensor $\{g^{ij}\}_{0\leq i,j\leq 3}$
(Note that for the Special Relativity case we also have
$\text{det}\,G=-1$). As for the cases of the General relativity, the
covariant formulation of Maxwell equations is still similar to
\er{khjhhkfgjfjhgghhgjghjhjkkkkgjghghuiiiulkkjlkklKKgfgjhjjghgjhhjhjhhhhhghhgtyytojjjjjuyiyuuint},
however, in addition to the different choice of the pseudo-metric
tensor $\{g^{ij}\}_{0\leq i,j\leq 3}$ we also have
$\text{det}\,G\neq Const.$ and thus for the full analogy equations
\er{khjhhkfgjfjhgghhgjghjhjkkkkgjghghuiiiulkkjlkklKKgfgjhjjghgjhhjhjhhhhhghhgtyytojjjjjuyiyuuint}
should be rewritten in the enlarged form, due to
\er{MaxVacFullPPNhjjghjjkjhhoujiiikjjihjhiuiuinthhint}:
\begin{multline}\label{khjhhkfgjfjhgghhgjghjhjkkkkgjghghuiiiulkkjlkklKKgfgjhjjghgjhhjhjhhhhhghhgtyytojjjjjuyiyuughgfgint}
\sum_{j=0}^{3}\frac{\partial}{\partial
x^j}\left(\sum_{m=0}^{3}\sum_{n=0}^{3}g^{km}g^{jn}\left(\frac{\partial
A_n}{\partial x^m}-\frac{\partial A_m}{\partial
x^n}\right)\right)+\\
\sum_{j=0}^{3}\frac{1}{\sqrt{|\text{det}\,G|}}\frac{\partial}{\partial
x^j}\left(\sqrt{|\text{det}\,G|}\right)\left(\sum_{m=0}^{3}\sum_{n=0}^{3}g^{km}g^{jn}\left(\frac{\partial
A_n}{\partial x^m}-\frac{\partial A_m}{\partial x^n}\right)\right)
=-4\pi j^k\quad\quad\forall\, k=0,1,2,3.
\end{multline}
Note also that we can rewrite
\er{khjhhkfgjfjhgghhgjghjhjkkkkgjghghuiiiulkkjlkklKKgfgjhjjghgjhhjhjhhhhhghhgtyytojjjjjuyiyuughgfgint}
as:
\begin{equation}\label{khjhhkfgjfjhgghhgjghjhjkkkkgjghghuiiiulkkjlkklKKgfgjhjjghgjhhjhjhhhhhghhgtyytojjjjjuyiyuughgfghhjint}
\sum_{j=0}^{3}\frac{\partial}{\partial
x^j}\left(\sum_{m=0}^{3}\sum_{n=0}^{3}\sqrt{|\text{det}\,G|}\,g^{km}g^{jn}\left(\frac{\partial
A_n}{\partial x^m}-\frac{\partial A_m}{\partial x^n}\right)\right)
=-4\pi \sqrt{|\text{det}\,G|}\, j^k\quad\quad\forall\, k=0,1,2,3.
\end{equation}

Next by
\er{MaxVacFullPPNhjjghjjkjhh1int} we have
\begin{equation}\label{MaxVacFullPPNhjjghjjkjhhiuyyint}
\frac{1}{2}|\vec D|^2-\frac{1}{2}|\vec
B|^2=-\sum_{j=0}^{3}\sum_{k=0}^{3}\frac{1}{4}F^{jk}F_{jk}.
\end{equation}
Therefore, by
\er{fgjfjhgghhgjghjhjijhojihjhjjijhjjjjjuiijjjklihhojjjoouuoiuiuint},
\er{fgjfjhgghhgjghjhjijhojihjhjjijhjjjljljpkyuuyyuint} and
\er{MaxVacFullPPNhjjghjjkjhhiuyyint}, we can rewrite the density of
the Lagrangian of the electromagnetic field, defined in
\er{vhfffngghkjgghPPNint} as
\begin{equation}\label{vhfffngghkjgghPPNggjgjjkgjoiuiint}
L_1\left(\vec A,\Psi,\vec
x,t\right):=\frac{1}{4\pi}\left(\frac{1}{2}\left|\vec
D\right|^2-\frac{1}{2}\left|\vec
B\right|^2-4\pi\left(\rho\Psi-\frac{1}{c}\vec A\cdot\vec
j\right)\right),
\end{equation}
in the equivalent covariant form:
\begin{multline}\label{vhfffngghkjgghPPNggjgjjkgjoiuioigghint}
L_1=\frac{1}{4\pi}\left(-\sum_{n=0}^{3}\sum_{k=0}^{3}\frac{1}{4}F^{nk}F_{nk}-\sum_{k=0}^{3}4\pi
j^k A_k\right)=\\
\frac{1}{4\pi}\left(-\sum_{n=0}^{3}\sum_{k=0}^{3}\sum_{m=0}^{3}\sum_{p=0}^{3}\frac{1}{4}g^{mn}g^{pk}\left(\frac{\partial
A_p}{\partial x^m}-\frac{\partial A_m}{\partial
x^p}\right)\left(\frac{\partial A_k}{\partial x^n}-\frac{\partial
A_n}{\partial x^k}\right)-\sum_{k=0}^{3}4\pi j^k A_k\right).
\end{multline}
The density of Lagrangian in
\er{vhfffngghkjgghPPNggjgjjkgjoiuioigghint} is also fully analogous
to the covariant formulation of the Lagrangian density of the
electromagnetic field in Special and General Relativity and the only
difference is the choice of the pseudo-metric tensor
$\{g^{ij}\}_{0\leq i,j\leq 3}$.

\subsubsection{Covariant formulation of Lagrangian of motion of a
classical charged particle in the external gravitational and
electromagnetic fields} Given a classical charged particle with
inertial mass $m$, charge $\sigma$, three-dimensional place $\vec
r(t)$ and three-dimensional velocity $\frac{d\vec r}{dt}$ in the
outer gravitational field with three-dimensional vectorial potential
$\vec v(\vec x,t)$, the outer electromagnetical field with
three-dimensional vectorial potential $\vec A(\vec x,t)$ and scalar
potential $\vec \Psi(\vec x,t)$, consider a usual Lagrangian that is
a particular case of \er{vhfffngghkjgghfjjint}:
\begin{equation}\label{vhfffngghkjgghfjjSYSPNkoijjhpoiint}
L_0\left(\frac{d\vec r}{dt},t\right):=
\left\{\frac{m}{2}\left|\frac{d\vec r}{dt}-\vec v(\vec
r,t)\right|^2-\sigma\left(\Psi(\vec r,t)-\frac{1}{c}\vec A(\vec
r,t)\cdot\frac{d\vec r}{dt}\right)\right\}.
\end{equation}
Then, since we are interesting in critical points of the functional
\begin{equation}\label{btfffygtgyggyijhhkkSYSPNuhuygyygyggyint}
J_0=\int_0^T L_0\left(\frac{d\vec r}{dt},\vec r,t\right)dt,
\end{equation}
adding a constant does not changes the physical meaning of the
Lagrangian and we can rewrite
\er{vhfffngghkjgghfjjSYSPNkoijjhpoiint} as:
\begin{equation}\label{vhfffngghkjgghfjjSYSPNkoijjhpoiuuiint}
L'_0\left(\frac{d\vec r}{dt},t\right):=
\left\{\left(\frac{m}{2}\left|\frac{d\vec r}{dt}-\vec v(\vec
r,t)\right|^2-\frac{mc^2}{2}\right)-\sigma\left(\Psi(\vec
r,t)-\frac{1}{c}\vec A(\vec r,t)\cdot\frac{d\vec
r}{dt}\right)\right\}.
\end{equation}
and \er{btfffygtgyggyijhhkkSYSPNuhuygyygyggyint} as
\begin{multline}\label{btfffygtgyggyijhhkkSYSPNuhuygyygyggyuyyint}
J'_0:=J_0-\frac{Tmc^2}{2}=\int_0^T L'_0\left(\frac{d\vec r}{dt},\vec
r,t\right)dt=\\ \int_0^T\left\{\left(\frac{m}{2}\left|\frac{d\vec
r}{dt}-\vec v(\vec
r,t)\right|^2-\frac{mc^2}{2}\right)-\sigma\left(\Psi(\vec
r,t)-\frac{1}{c}\vec A(\vec r,t)\cdot\frac{d\vec
r}{dt}\right)\right\}dt,
\end{multline}
Next consider the four-vector field of the momentum on the group
$\mathcal{S}_0$: $\left(p^0(t),p^1(t),p^2(t),p^3(t)\right)$, defined
by \er{fgjfjhgghhgjghjhjijhojihjhjjijhjjjjjuiijhjhhint} and
\er{fgjfjhgghhgjghjhjijhojihjhjjijhjjjjjuiijjjklihhint} as:
\begin{equation}\label{fgjfjhgghhgjghjhjijhojihjhjjijhjjjjjuiijhjhhioiint}
\left(
p^0(t),p^1(t),p^2(t),p^3(t)\right):=\left(m,\frac{m}{c}\frac{d\vec
r}{dt}(t)\right)=
\left(m,\frac{m}{c}\frac{dr_1}{dt}(t),\frac{m}{c}\frac{dr_2}{dt}(t),\frac{m}{c}\frac{dr_3}{dt}(t)\right)
\end{equation}
Then by
\er{fgjfjhgghhgjghjhjkkkkgjghghuiiiulkkjKKyuyyu0ioioiogghghghgghghint}
we have
\begin{multline}\label{fgjfjhgghhgjghjhjkkkkgjghghuiiiulkkjKKyuyyu0ioioiogghghghgghghhghgvint}
\frac{mc^2}{2}\left(\frac{1}{c^2}\left|\frac{d\vec r}{dt}-\vec
v(\vec r,t)\right|^2-1\right)=\frac{m}{2}\left|\frac{d\vec
r}{dt}-\vec v(\vec
r,t)\right|^2-\frac{mc^2}{2}\\=-\frac{c^2}{2m}\left(\sum_{k=0}^{3}p^kp_k\right)=-\frac{mc^2}{2}\left(\sum_{j=0}^{3}\sum_{k=0}^{3}g_{jk}(\vec
r,t)\,\frac{p^j}{m}\,\frac{p^k}{m}\right).
\end{multline}
On the other hand if we consider the four-covector of the
electromagnetic potential on the group $\mathcal{S}_0$:
$(A_0,A_1,A_2,A_3)$, defined by
\er{fgjfjhgghhgjghjhjijhojihjhjjijhjjjljljpkint} as:
\begin{equation}\label{fgjfjhgghhgjghjhjijhojihjhjjijhjjjljljpkyuuyyuhhhhjint}
(A_0,A_1,A_2,A_3)=(\Psi,-\vec A),
\end{equation}
then we can write,
\begin{equation}\label{btfffygtgyggyijhhkkSYSPNuhuygyygyggyuyyioiooiihhhjint}
\sigma\left(\Psi(\vec r,t)-\frac{1}{c}\vec A(\vec
r,t)\cdot\frac{d\vec r}{dt}\right)=\sum_{k=0}^{3}\sigma A_k(\vec
r,t)\,\frac{p^k}{m}.
\end{equation}
Thus by
\er{fgjfjhgghhgjghjhjkkkkgjghghuiiiulkkjKKyuyyu0ioioiogghghghgghghhghgvint}
and \er{btfffygtgyggyijhhkkSYSPNuhuygyygyggyuyyioiooiihhhjint} we
rewrite \er{btfffygtgyggyijhhkkSYSPNuhuygyygyggyuyyint} in a
covariant form:
\begin{multline}\label{btfffygtgyggyijhhkkSYSPNuhuygyygyggyuyyuyuyint}
J'_0=\int_0^T L'_0\left(\frac{d\vec r}{dt},\vec r,t\right)dt=
\int_0^T\left\{-\frac{mc^2}{2}\left(\sum_{j=0}^{3}\sum_{k=0}^{3}g_{jk}(\vec
r,t)\,\frac{p^j}{m}\,\frac{p^k}{m}\right)-\sum_{k=0}^{3}\sigma
A_k(\vec r,t)\,\frac{p^k}{m}\right\}dt.
\end{multline}
Thus if we consider the four-dimensional space-time trajectory of
the particle:
\begin{equation}\label{btfffjhgjghghint}
\left(\chi^0(t),\chi^1(t),\chi^2(t),\chi^3(t)\right)=\left(t,\frac{1}{c}r_1(t),\frac{1}{c}r_2(t),\frac{1}{c}r_3(t)\right),
\end{equation}
then we rewrite \er{btfffygtgyggyijhhkkSYSPNuhuygyygyggyuyyuyuyint}
as:
\begin{equation}\label{btfffygtgyggyijhhkkSYSPNuhuygyygyggyuyyuyuykuhghgint}
J'_0=
\int_0^T\left\{-\frac{mc^2}{2}\left(\sum_{j=0}^{3}\sum_{k=0}^{3}g_{jk}\left(\chi(t)\right)\,\frac{d\chi^j}{dt}\,\frac{d\chi^k}{dt}\right)-\sum_{k=0}^{3}\sigma
A_k\left(\chi(t)\right)\,\frac{d\chi^k}{dt}\right\}dt.
\end{equation}
Moreover, $\left(\frac{d\chi^0}{dt}, \frac{d\chi^1}{dt},
\frac{d\chi^2}{dt}, \frac{d\chi^3}{dt}\right)$ is a four-vector on
the group $\mathcal{S}_0$ and the global non-relativistic time $t$
is the scalar on the group $\mathcal{S}_0$.

Next we also can consider a more general Lagrangian than
\er{btfffygtgyggyijhhkkSYSPNuhuygyygyggyuyyuyuykuhghgint}: given a
function $\mathcal{G}(\tau):\mathbb{R}\to\mathbb{R}$ define:
\begin{equation}\label{btfffygtgyggyijhhkkSYSPNuhuygyygyggyuyyuyuykuhghgjjojint}
J_{\mathcal{G}}(\chi)=
\int_0^T\left\{-mc^2\;\mathcal{G}\left(\sum_{j=0}^{3}\sum_{k=0}^{3}g_{jk}\left(\chi(t)\right)\,\frac{d\chi^j}{dt}\,\frac{d\chi^k}{dt}\right)-\sum_{k=0}^{3}\sigma
A_k\left(\chi(t)\right)\,\frac{d\chi^k}{dt}\right\}dt.
\end{equation}
Clearly,
\er{btfffygtgyggyijhhkkSYSPNuhuygyygyggyuyyuyuykuhghgjjojint} is
written in covariant form, and in particular,
\er{btfffygtgyggyijhhkkSYSPNuhuygyygyggyuyyuyuykuhghgjjojint} is
invariant under the change of non-inertial cartesian coordinate
systems. In particular, for $\mathcal{G}(\tau):=\frac{1}{2}\tau$ we
obtain \er{btfffygtgyggyijhhkkSYSPNuhuygyygyggyuyyuyuykuhghgint}.

Another important particular case is the following choice:
$\mathcal{G}(\tau):=\sqrt{\tau}$. Then we deduce:
\begin{equation}\label{btfffygtgyggyijhhkkSYSPNuhuygyygyggyuyyuyuykuhghgjjojiyyint}
J_{rl}(\chi)=
\int_0^T\left\{-mc^2\;\sqrt{\left(\sum_{j=0}^{3}\sum_{k=0}^{3}g_{jk}\left(\chi(t)\right)\,\frac{d\chi^j}{dt}\,\frac{d\chi^k}{dt}\right)}-\sum_{k=0}^{3}\sigma
A_k\left(\chi(t)\right)\,\frac{d\chi^k}{dt}\right\}dt,
\end{equation}
that is in somewhat analogous to the relativistic Lagrangian of the
motion of charged particle. Due to \er{btfffjhgjghghint} we rewrite
\er{btfffygtgyggyijhhkkSYSPNuhuygyygyggyuyyuyuykuhghgjjojiyyint} in
a three-dimensional form as:
\begin{equation}\label{btfffygtgyggyijhhkkSYSPNuhuygyygyggyuyyuyuykuhghgjjojiyyyuyuuyyint}
J_{rl}(\vec r)= \int_0^T\left\{
-mc^2\sqrt{1-\frac{1}{c^2}\left|\frac{d\vec r}{dt}-\vec v(\vec
r,t)\right|^2}-\sigma\left(\Psi(\vec r,t)-\frac{1}{c}\vec A(\vec
r,t)\cdot\frac{d\vec r}{dt}\right)\right\}dt.
\end{equation}
Thus in the case
$$\frac{1}{c^2}\left|\frac{d\vec r}{dt}-\vec v(\vec
r,t)\right|^2\ll 1,$$ up to additive constant,
\er{btfffygtgyggyijhhkkSYSPNuhuygyygyggyuyyuyuykuhghgjjojiyyyuyuuyyint}
becomes to be \er{btfffygtgyggyijhhkkSYSPNuhuygyygyggyint}, where
$L_0$ is given by \er{vhfffngghkjgghfjjSYSPNkoijjhpoiint}. Note that
the Lagrangian in
\er{btfffygtgyggyijhhkkSYSPNuhuygyygyggyuyyuyuykuhghgjjojiyyint} has
the following advantage with respect to
\er{btfffygtgyggyijhhkkSYSPNuhuygyygyggyuyyuyuykuhghgint}: if we
parameterize the curve in \er{btfffjhgjghghint} by some arbitrary
parameter $s$ with increasing mapping $t\leftrightarrow s$, which is
however can differ from the global time $t$, then changing variables
of integration in
\er{btfffygtgyggyijhhkkSYSPNuhuygyygyggyuyyuyuykuhghgjjojiyyint}
from $t$ to $s$ gives:
\begin{equation}\label{btfffygtgyggyijhhkkSYSPNuhuygyygyggyuyyuyuykuhghgjjojiyyjljlghint}
J_{rl}(\chi)=
\int_a^b\left\{-mc^2\;\sqrt{\left(\sum_{j=0}^{3}\sum_{k=0}^{3}g_{jk}\left(\chi(s)\right)\,\frac{d\chi^j}{ds}\,\frac{d\chi^k}{ds}\right)}-\sum_{k=0}^{3}\sigma
A_k\left(\chi(s)\right)\,\frac{d\chi^k}{ds}\right\}ds,
\end{equation}
that has exactly the same form as
\er{btfffygtgyggyijhhkkSYSPNuhuygyygyggyuyyuyuykuhghgjjojiyyint},
however $s$ in
\er{btfffygtgyggyijhhkkSYSPNuhuygyygyggyuyyuyuykuhghgjjojiyyjljlghint}
can be \underline{arbitrary} parameter of the curve with increasing
mapping $t\leftrightarrow s$.

Finally, we would like to note that if the motion of some particle
is ruled by the relativistic-like Lagrangian in
\er{btfffygtgyggyijhhkkSYSPNuhuygyygyggyuyyuyuykuhghgjjojiyyyuyuuyyint},
then, although the absolute value of the velocity of the particle
$\left|\frac{d\vec r}{dt}\right|$ can be arbitrary large, the
absolute value of the difference between the velocity of the
particle and the local gravitational potential cannot exceed the
value $c$, i.e.:
\begin{equation}\label{btfffygtgyggyijhhkkSYSPNuhuygyygyggyuyyuyuykuhghgjjojiyyyuyuuyyijyyuyuint}
\left|\vec u(t)-\vec v(\vec r,t)\right|:=\left|\frac{d\vec
r}{dt}-\vec v(\vec r,t)\right|\,<\,c\quad\quad\quad\quad\forall\, t,
\end{equation}
provided that
\er{btfffygtgyggyijhhkkSYSPNuhuygyygyggyuyyuyuykuhghgjjojiyyyuyuuyyijyyuyuint}
is satisfied in some initial instant of time. Note also that the
quantity in the left hand side of
\er{btfffygtgyggyijhhkkSYSPNuhuygyygyggyuyyuyuykuhghgjjojiyyyuyuuyyijyyuyuint}
is invariant under the change of inertial or non-inertial cartesian
coordinate system.

\subsubsection{Kinematic pseudo-metric tensors of
inertia}\label{jjjjjkkk88int} Consider $\{J^{ij}\}_{0\leq i,j\leq
3}$ to be a two times contravariant tensor field on the group
$\mathcal{S}_0$, defined by
\begin{equation}\label{hoyuiouigyfghgjh3478zzrrZZffGGjkkjojj88int}
J^{ij}:=k^ik^j-{\Theta}^{ij}\quad\quad\forall\,i,j=0,1,2,3,
\end{equation}
where $\{{\Theta}^{ij}\}_{0\leq i,j\leq 3}$ is the contravariant
tensor of the three-dimensional geometry, defined by
\er{huohuioy89gjjhjffffff3478zzrrZZZhjhhjhhjjhhffGGhjjhuiiuihhjkkoioiint}
and being a two times contravariant tensor, and $(k^0,k^1,k^2,k^3)$
is the four-dimensional potential of inertia, defined by
\er{fgjfjhgghhgjghjhjijhojihjhjjijhjjjjjuiijjjkjkkhhkkhint} and
being a four-vector. Then
in subsection \ref{jjjjjkkk88} we obtain that $\{J^{ij}\}_{0\leq
i,j\leq 3}$ is indeed a two times contravariant tensor field on the
group $\mathcal{S}_0$ and moreover, this tensor is symmetric.
Moreover, by
\er{huohuioy89gjjhjffffff3478zzrrZZZhjhhjhhjjhhffGGhjjhuiiuihhjkkoioiint}
and \er{fgjfjhgghhgjghjhjijhojihjhjjijhjjjjjuiijjjkjkkhhkkhint} we
have:
\begin{equation}\label{hoyuiouigyfghgjh3478zzrrZZffGGjkkj88int}
\begin{cases}
J^{00}=1\\
J^{ij}=-\delta_{ij}+\frac{k^ik^j}{c^2}\quad\forall 1\leq i,j\leq 3\\
J^{0j}=J^{j0}=\frac{k^j}{c}\quad\forall 1\leq j\leq 3,
\end{cases}
\end{equation}
where $\vec k=(k^1,k^2,k^3)$ is the three-dimensional vectorial
potential of inertia. We call the tensor $\{J^{ij}\}_{0\leq i,j\leq
3}$ the contravariant kinematic pseudo-metric tensor of inertia.
Next consider a $16$-component field $\{J_{ij}\}_{0\leq i,j\leq 3}$
defined by
\begin{equation}\label{hoyuiouigyfg3478zzrrZZffGGhhjhj88int}
\begin{cases}
J_{00}=1-\frac{|\vec k|^2}{c^2}\\
J_{ij}=-\delta_{ij}\quad\forall 1\leq i,j\leq 3\\
J_{0j}=J_{j0}=\frac{k^j}{c}\quad\forall 1\leq j\leq 3.
\end{cases}
\end{equation}
Then, as before in subsection \ref{jjjjjkkk88int}, we deduce
\begin{equation}\label{fgjfjhgghhgjghjhjkkkkgjghghuiiiuujhjh88int}
\sum_{m=0}^{3}J^{im}J_{mj}=\begin{cases}
1\quad\text{if}\quad i=j\\
0\quad\text{if}\quad i\neq j
\end{cases}\quad\quad\forall\, i,j=0,1,2,3.
\end{equation}
Therefore,
we obtain that $\{J_{ij}\}_{i,j=0,1,2,3}$ is a two times covariant
tensor on the group $\mathcal{S}_0$, and moreover, this tensor is
symmetric. We call the tensor $\{J_{ij}\}_{0\leq i,j\leq 3}$
covariant kinematic pseudo-metric tensors of inertia. Using
\er{fgjfjhgghhgjghjhjkkkkgjghghuiiiuujhjh88int} we also obtain that
the pseudo-metric tensors $\{J_{ij}\}_{i,j=0,1,2,3}$ and
$\{J^{ij}\}_{0\leq i,j\leq 3}$ are non-degenerate and they are
inverse of each other. Moreover, as before, it can be easily
calculated that if we consider the $4\times 4$-matrix:
\begin{equation}\label{fgjfjhgghhgjghjhjkkkkgjghghuiiiuujhjhjklj88int}
J=\{J_{ij}\}_{0\leq i,j\leq 3},
\end{equation}
then
\begin{equation}\label{fgjfjhgghhgjghjhjkkkkgjghghuiiiuujhjh188int}
\text{det}\,J=-1.
\end{equation}
In particular, by \er{fgjfjhgghhgjghjhjkkkkgjghghuiiiuujhjh188int}
and \er{fgjfjhgghhgjghjhjkkkkgjghghuiiiuujhjh1int} we deduce
\begin{equation}\label{fgjfjhgghhgjghjhjkkkkgjghghuiiiuujhjhljhgugii88int}
\text{det}\,G=\text{det}\,J.
\end{equation}
where $4\times 4$-matrix $G$ is given by,
\begin{equation}\label{fgjfjhgghhgjghjhjkkkkgjghghuiiiuujhjhjkljojuuiyi88int}
G=\{g_{ij}\}_{0\leq i,j\leq 3},
\end{equation}
with the covariant pseudo-metric tensor of the four-dimensional
space-time $\{g_{ij}\}_{0\leq i,j\leq 3}$, given by
\er{hoyuiouigyfg3478zzrrZZffGGhhjhjint}. Next, obviously we have
\begin{equation}\label{hoyuiouigyfghgjh3478zzrrZZffGGjkkjojj88iihhkhk88int}
{g}^{ij}=J^{ij}+\left(v^iv^j-k^ik^j\right)\quad\quad\forall\,i,j=0,1,2,3,
\end{equation}
where $J^{ij}$ the is the contravariant kinematic pseudo-metric
tensor of inertia, given by
\er{hoyuiouigyfghgjh3478zzrrZZffGGjkkjojj88int}, and $g^{ij}$ is the
contravariant pseudo-metric tensor of the four-dimensional
space-time, given by \er{hoyuiouigyfghgjh3478zzrrZZffGGjkkjojjint}.
In particular, in the case of the absence of the genuine gravity
where $\vec v=\vec k$ we have
\begin{equation}\label{hoyuiouigyfghgjh3478zzrrZZffGGjkkjojj88iihhkhk88jouyiyi88int}
{g}^{ij}=J^{ij}\quad\quad\forall\,i,j=0,1,2,3,
\end{equation}
or equivalently,
\begin{equation}\label{hoyuiouigyfghgjh3478zzrrZZffGGjkkjojj88iihhkhk88jouyiyi88joljljkj88int}
{g}_{ij}=J_{ij}\quad\quad\forall\,i,j=0,1,2,3.
\end{equation}

Furthermore, by \er{hoyuiouigyfghgjh3478zzrrZZffGGjkkj88int} we have
\begin{equation}\label{fgjfjhgghhgjghjhjijhojihjhjjijhjjjjjuiijjjkhjhjhjuiiuuuyu88jukukukuk88kpjljlint}
\sum_{m=0}^{3} \,J^{jm}\,\left(c\,\frac{\partial\tau}{\partial
x^m}\right)=k^j\quad\quad\forall\,j=0,1,2,3.
\end{equation}
where $(k^0,k^1,k^2,k^3)$ is the four-dimensional potential of
inertia, defined by
\er{fgjfjhgghhgjghjhjijhojihjhjjijhjjjjjuiijjjkjkkhhkkhint} and
$\tau$ is the scalar of the global time on the group
$\mathcal{S}_0$, defined by \er{uguyytfddddgghkjhhjuuiuiint}. Thus
since,
\begin{equation}\label{fgjfjhgghhgjghjhjijhojihjhjjijhjjjjjuiijjjkhjhjhjuiiuuuyu88jukukukuk88kpjljlkjojint}
\sum_{j=0}^{3} \,k^j\,\left(c\,\frac{\partial\tau}{\partial
x^j}\right)=1,
\end{equation}
we deduce the following eikonal type equation
\begin{equation}\label{fgjfjhgghhgjghjhjijhojihjhjjijhjjjjjuiijjjkhjhjhjuiiuuuyu88int}
c^2\left(\sum_{j=0}^{3}\sum_{m=0}^{3}\,J^{jm}\frac{\partial\tau}{\partial
x^j}\frac{\partial\tau}{\partial x^m}\right)=1,
\end{equation}
which is similar to
\er{fgjfjhgghhgjghjhjijhojihjhjjijhjjjjjuiijjjkhjhjhjuiiuuuyuint},
despite of the different pseudometrics. Moreover, by
\er{fgjfjhgghhgjghjhjijhojihjhjjijhjjjjjuiijjjkhjhjhjuiiuuuyu88jukukukuk88kpjljlint}
and \er{fgjfjhgghhgjghjhjkkkkgjghghuiiiuujhjh88int} we infer:
\begin{equation}\label{fgjfjhgghhgjghjhjijhojihjhjjijhjjjjjuiijjjkhjhjhjuiiuuuyu88jukukukuk88kpjljlmhmbbvnvvvkljljl88int}
\sum_{m=0}^{3} \,J_{jm}\,k^m\,=\,c\,\frac{\partial\tau}{\partial
x^j}\quad\quad\forall\,j=0,1,2,3.
\end{equation}

 Next, as before, note that  the Kinematic pseudo-metric tensors of inertia $\{J^{ij}\}_{0\leq i,j\leq 3}$ and $\{J_{ij}\}_{0\leq i,j\leq 3}$  depend only on
the coordinate system in the space-time and are completely
independent of the physical matter or physical fields filling this
space. In contrast the pseudo-metric tensors of the four-dimensional
space-time $\{g^{ij}\}_{0\leq i,j\leq 3}$ and $\{g_{ij}\}_{0\leq
i,j\leq 3}$, given by \er{hoyuiouigyfghgjh3478zzrrZZffGGjkkjint} and
\er{hoyuiouigyfg3478zzrrZZffGGhhjhjint}, depend essentially on the
surrounding physical matter (in the model of the Newtonian gravity
through gravitational masses). Furthermore, note that in the absence
of the genuine gravity $\{g^{ij}\}_{0\leq i,j\leq 3}$ and
$\{g_{ij}\}_{0\leq i,j\leq 3}$ coincide with $\{J^{ij}\}_{0\leq
i,j\leq 3}$ and $\{J_{ij}\}_{0\leq i,j\leq 3}$ respectively. In the
model of the Newtonian gravity this happens away from essential
gravitational masses. In particular, $\{g^{ij}\}_{0\leq i,j\leq 3}$
and $\{g_{ij}\}_{0\leq i,j\leq 3}$ tend to $\{J^{ij}\}_{0\leq
i,j\leq 3}$ and $\{J_{ij}\}_{0\leq i,j\leq 3}$ respectively, as
$|\vec x|\to +\infty$. Finally note that since in every inertial
cartesian coordinate system $\vec k$ is a constant we deduce that in
every such a system $\{J^{ij}\}_{0\leq i,j\leq 3}$ and
$\{J_{ij}\}_{0\leq i,j\leq 3}$ are \underline{constant} in the
four-dimensional space-time. So, in contrast to $\{g^{ij}\}_{0\leq
i,j\leq 3}$ and $\{g_{ij}\}_{0\leq i,j\leq 3}$, the pseudo-metrics
$\{J^{ij}\}_{0\leq i,j\leq 3}$ and $\{J_{ij}\}_{0\leq i,j\leq 3}$
are complectly \underline{flat}. In particular, if we consider the
Christoffel Symbols of the Kinematic pseudo-metric tensors of
inertia $\{J_{ij}\}_{0\leq i,j\leq 3}$ defined by,
\begin{equation}\label{jhffhgffgh3478zzrrZZffjjkjkjklkkj88int}
\begin{cases}
\left\{\Gamma_{i,mn}\right\}_J:=\frac{1}{2}\left(\frac{\partial
J_{im}}{\partial x_{n}}+\frac{\partial J_{in}}{\partial
x_{m}}-\frac{\partial
J_{mn}}{\partial x_{i}}\right)\\
\left\{\Gamma^{i}_{mn}\right\}_J:=\sum_{j=0}^{3}J^{ij}\left\{\Gamma_{j,mn}\right\}_J
\end{cases}\quad\quad\forall\, i,m,n=0,1,2,3,
\end{equation}
then, since in every inertial cartesian coordinate system
$\{J_{ij}\}_{0\leq i,j\leq 3}$ is \underline{constant} in the
four-dimensional space-time, we deduce that in every such coordinate
system we have
\begin{equation}\label{jhffhgffghhjgj3478zzrrZZffkhkhkljjllkjl88int}
\left\{\Gamma^{i}_{mn}\right\}_J=\left\{\Gamma_{i,mn}\right\}_J=0\quad\quad\forall\,
i,m,n=0,1,2,3.
\end{equation}
Thus, if we consider the two times covariant tensor of the covariant
derivative to the covector of the gradient to the scalar of the
global time $\tau$ from \er{uguyytfddddgghkjhhjuuiuiint}, denoted by
$\left\{\delta_j \left(\frac{\partial\tau}{\partial x_i}
\right)\right\}_J$, and defined by:
\begin{multline}\label{hoyuiouigyfghgjh3478zzrrZZjhjhugyuyughggjhjhjyuuygtyffjjjkjljkjjljokjllllklljkljkjlnnnljjl88int}
\left\{\delta_j \left(\frac{\partial\tau}{\partial x_i}
\right)\right\}_J=\left\{\delta_i \left(\frac{\partial\tau}{\partial
x_j} \right)\right\}_J:=\frac{\partial^2\tau}{\partial x_i\partial
x_j}-\sum_{m=0}^{3}\left\{\Gamma^{m}_{ij}\right\}_J\,\frac{\partial
\tau}{\partial x_m}
\\=\frac{\partial^2\tau}{\partial x_i\partial
x_j}-\sum_{m=0}^{3}\left\{\Gamma_{m,ij}\right\}_J\,
\frac{k^m}{c}\quad\quad\forall\, i,j=0,1,2,3\,,
\end{multline}
then by
\er{fgjfjhgghhgjghjhjkkkkgjghghuiiiulkkjlkklplikklkjljghhgghkint}
and \er{jhffhgffghhjgj3478zzrrZZffkhkhkljjllkjl88int} we prove the
following identity, first in every inertial cartesian coordinate
system and then, by the covariance of this identity, also in every
non-inertial cartesian coordinate system:
\begin{equation}\label{hoyuiouigyfghgjh3478zzrrZZjhjhugyuyughggjhjhjyuuygtyffjjjkjljkjjljokjllllklljkljkjlnnnhkhkhkhkhh88int}
\left\{\delta_j \left(\frac{\partial\tau}{\partial x_i}
\right)\right\}_J=0\quad\quad\forall\, i,j=0,1,2,3\,.
\end{equation}
Note here that we put brackets $\left\{\cdot\right\}_J$ for the
Christoffel Symbols and covariant derivative with respect to
Kinematical pseudo-metrics  $\{J^{ij}\}_{0\leq i,j\leq 3}$ and
$\{J_{ij}\}_{0\leq i,j\leq 3}$ in order to distinguish them from
Christoffel Symbols and covariant derivative with respect to
Dynamical pseudo-metrics  $\{g^{ij}\}_{0\leq i,j\leq 3}$ and
$\{g_{ij}\}_{0\leq i,j\leq 3}$, that we will write
\underline{without} any brackets.

Furthermore, since $\vec k$ is \underline{generally trivial}
speed-like vector field, by the definition, there exists a unique,
up to equivalence, "preferable" inertial cartesian coordinate system
($\{0\}$) where we have
\begin{equation}\label{fgjfjhgghhgjghjhjijhojihjhjjijhjjjjjuiijjjkhjhjhjuiiuuuyu88jkhmbmbm88int}
\vec k=0,
\end{equation}
and inserting it into \er{hoyuiouigyfghgjh3478zzrrZZffGGjkkj88int}
and \er{hoyuiouigyfg3478zzrrZZffGGhhjhj88int} we deduce that in
system ($\{0\}$) we have
\begin{equation}\label{hoyuiouigyfghgjh3478zzrrZZffGGjkkj88jokhjh88int}
\begin{cases}
J^{00}=1\\
J^{ij}=-\delta_{ij}\quad\forall 1\leq i,j\leq 3\\
J^{0j}=J^{j0}=0\quad\forall 1\leq j\leq 3,
\end{cases}
\end{equation}
and
\begin{equation}\label{hoyuiouigyfghgjh3478zzrrZZffGGjkkj88jokhjh88jkhkjhjh88int}
\begin{cases}
J_{00}=1\\
J_{ij}=-\delta_{ij}\quad\forall 1\leq i,j\leq 3\\
J_{0j}=J_{j0}=0\quad\forall 1\leq j\leq 3,
\end{cases}
\end{equation}
which is the same as the basic pseudo-metrics in the Special
Relativity.

Next, we define the Dynamical four-covector of genuine gravity
$(s_0,s_1,s_2,s_3)$ by the following:
\begin{equation}\label{huohuioy89gjjhjffffff3478zzrrZZZhjhhjhhjjhhffGGhjjhuiiuihhjkkoioirrkhrrjggggjrrjiokuukuppkpklhjjhj88int}
s_j=\frac{1}{2}\left(\sum_{m=0}^{3} g_{jm}k^m-\sum_{m=0}^{3}
J_{jm}v^m\right)\quad \forall j=0,1,2,3,
\end{equation}
where $(k^0,k^1,k^2,k^3)$ is the four-dimensional potential of
inertia, defined by
\er{fgjfjhgghhgjghjhjijhojihjhjjijhjjjjjuiijjjkjkkhhkkhint},
$(v^0,v^1,v^2,v^3)$ is the four-dimensional gravitational potential,
defined by \er{fgjfjhgghhgjghjhjijhojihjhjjijhjjjjjuiijjjkint},
$\{g_{ij}\}_{0\leq i,j\leq 3}$ defined by
\er{hoyuiouigyfg3478zzrrZZffGGhhjhjint} and $\{J_{ij}\}_{0\leq
i,j\leq 3}$ defined by \er{hoyuiouigyfg3478zzrrZZffGGhhjhj88int}.
Then, by inserting
\er{fgjfjhgghhgjghjhjijhojihjhjjijhjjjjjuiijjjkjkkhhkkhint},
\er{fgjfjhgghhgjghjhjijhojihjhjjijhjjjjjuiijjjkint},
\er{hoyuiouigyfg3478zzrrZZffGGhhjhjint} and
\er{hoyuiouigyfg3478zzrrZZffGGhhjhj88int} into
\er{huohuioy89gjjhjffffff3478zzrrZZZhjhhjhhjjhhffGGhjjhuiiuihhjkkoioirrkhrrjggggjrrjiokuukuppkpklhjjhj88int}
we obtain
%
%
%
%
%
%
%
%
\begin{multline}\label{fgjfjhgghhgjghjhjijhojihjhjjijhjjjjjuiijjjk88int}
(s_0,s_1,s_2,s_3)=\left(-\frac{1}{2c^2}|\vec h|^2-\frac{1}{c^2}\vec
k\cdot\vec h\,,\,\frac{1}{c}\vec h\right)=\left(\frac{1}{2c^2}|\vec
h|^2-\frac{1}{c^2}\vec v\cdot\vec h\,,\,\frac{1}{c}\vec
h\right)\quad\quad\text{where}\quad\\ s_0=-\frac{1}{2c^2}|\vec
h|^2-\frac{1}{c^2}\vec k\cdot\vec h=\frac{1}{2c^2}|\vec
h|^2-\frac{1}{c^2}\vec v\cdot\vec
h\;\;\quad\text{and}\quad\;\;(s_1,s_2,s_3)=\frac{1}{c}\vec h,
\end{multline}
where $\vec h$ is the three-dimensional proper vector field of the
vectorial potential of genuine gravity defined as in
\er{jljjjjkhhhjhjhjhljjjljljkjkCChhhhint} by
\begin{equation}\label{huohuioy89gjjhjffffff3478zzrrZZZhjhhjhhjjhhffGGhjjhuiiuihhjkkoioirrkhrrjggggjrrjiokuukuppkpklhjjhjljoyhkiogu88int}
\vec h:=\vec v-\vec k,
\end{equation}
(see subsection \ref{jjjjjkkk88} for details). In particular, by
\er{hoyuiouigyfg3478zzrrZZffGGhhjhjint},
\er{hoyuiouigyfg3478zzrrZZffGGhhjhj88int},
\er{fgjfjhgghhgjghjhjijhojihjhjjijhjjjjjuiijjjk88int} and
\er{fgjfjhgghhgjghjhjkkkkgjghghuiiiulkkjlkklplikklkjljghhgghkint}
together we deduce:
\begin{equation}\label{hoyuiouigyfghgjh3478zzrrZZffGGjkkjojj88iihhkhk88jklhkhjkhk88int}
{g}_{ij}=J_{ij}+\left(s_i\left(c\,\frac{\partial\tau}{\partial
x^j}\right)+\left(c\,\frac{\partial\tau}{\partial
x^i}\right)s_j\right)\quad\quad\forall\,i,j=0,1,2,3.
\end{equation}
Moreover, by \er{fgjfjhgghhgjghjhjijhojihjhjjijhjjjjjuiijjjk88int}
\er{huohuioy89gjjhjffffff3478zzrrZZZhjhhjhhjjhhffGGhjjhuiiuihhjkkoioiint}
and \er{hoyuiouigyfghgjh3478zzrrZZffGGjkkjojj88int} we obtain:
\begin{equation}\label{hoyuiouigyfghgjh3478zzrrZZffGGjkkjojj88jklkhhk88khjhgh88int}
v^j-k^j=\sum_{m=0}^{3}{\Theta}^{jm}s_m=\sum_{m=0}^{3}\left(k^jk^m-J^{jm}\right)s_m\quad\quad\forall\,j=0,1,2,3.
\end{equation}

\subsubsection{Physical laws in curvilinear coordinate systems in the
non-relativistic space-time}
Let $\mathcal{S}$ be the group of all smooth non-degenerate
invertible transformations from $\mathbb{R}^4$ onto $\mathbb{R}^4$
having the form \er{fgjfjhgghyuyyuint}:
\begin{equation}\label{fgjfjhgghyuyyuyughgint}
\begin{cases}
x'^0=f^{(0)}(x^0,x^1,x^2,x^3),\\
x'^1=f^{(1)}(x^0,x^1,x^2,x^3),\\
x'^2=f^{(2)}(x^0,x^1,x^2,x^3),\\
x'^3=f^{(3)}(x^0,x^1,x^2,x^3),
\end{cases}
\end{equation}
and let $\mathcal{S}_0$ be a subgroup of transformations of the form
\er{noninchgravortbstrjgghguittu2intrrrZZygjyghhjint}. Then, it is
clear, that given any object that is a scalar, four-vector,
four-covector, two-times covariant tensor or two-times contravariant
tensor on the group $\mathcal{S}_0$, defined in every cartesian
non-inertial coordinate system, we can uniquely extend the
definition of this object, in such a way that it will be defined
also in every curvilinear coordinate systems in $\mathbb{R}^4$ and
will be respectively a scalar, four-vector, four-covector, two-times
covariant tensor or two-times contravariant tensor on the wider
group $\mathcal{S}$. Thus all the physical laws that have a
covariant form preserve their form also in transformations of the
form \er{fgjfjhgghyuyyuyughgint} i.e. in curvilinear coordinate
systems. In particular, the Maxwell Equations in every curvilinear
coordinate system have the form of
\er{khjhhkfgjfjhgghhgjghjhjkkkkgjghghuiiiulkkjlkklKKgfgjhjjghgjhhjhjhhhhhghhgtyytojjjjjuyiyuughgfgint}
or equivalently of
\er{khjhhkfgjfjhgghhgjghjhjkkkkgjghghuiiiulkkjlkklKKgfgjhjjghgjhhjhjhhhhhghhgtyytojjjjjuyiyuughgfghhjint}:
\begin{multline}\label{khjhhkfgjfjhgghhgjghjhjkkkkgjghghuiiiulkkjlkklKKgfgjhjjghgjhhjhjhhhhhghhgtyytojjjjjuyiyuughgfghjjhkpkint}
\sum_{j=0}^{3}\frac{\partial}{\partial
x^j}\left(\sum_{m=0}^{3}\sum_{n=0}^{3}g^{km}g^{jn}\left(\frac{\partial
A_n}{\partial x^m}-\frac{\partial A_m}{\partial
x^n}\right)\right)+\\
\sum_{j=0}^{3}\frac{1}{\sqrt{|\text{det}\,G|}}\frac{\partial}{\partial
x^j}\left(\sqrt{|\text{det}\,G|}\right)\left(\sum_{m=0}^{3}\sum_{n=0}^{3}g^{km}g^{jn}\left(\frac{\partial
A_n}{\partial x^m}-\frac{\partial A_m}{\partial x^n}\right)\right)
=-4\pi j^k\quad\quad\forall\, k=0,1,2,3,
\end{multline}
or equivalently:
\begin{equation}\label{khjhhkfgjfjhgghhgjghjhjkkkkgjghghuiiiulkkjlkklKKgfgjhjjghgjhhjhjhhhhhghhgtyytojjjjjuyiyuughgfghhjkkhjint}
\sum_{j=0}^{3}\frac{\partial}{\partial
x^j}\left(\sum_{m=0}^{3}\sum_{n=0}^{3}\sqrt{|\text{det}\,G|}\,g^{km}g^{jn}\left(\frac{\partial
A_n}{\partial x^m}-\frac{\partial A_m}{\partial x^n}\right)\right)
=-4\pi \sqrt{|\text{det}\,G|}\, j^k\quad\quad\forall\, k=0,1,2,3.
\end{equation}
Here $\{A_k\}_{k=0,1,2,3}$ is the four-covector of the
electromagnetic potential, $\{j^k\}_{k=0,1,2,3}$ is the four-vector
of the current and $G:=\{g_{kj}\}_{k,j=0,1,2,3}$,
$\{g^{kj}\}_{k,j=0,1,2,3}$ are pseudo-metric covariant and
contravariant tensors. Note, that in curvilinear coordinate system
we can have $\text{det}\,G\neq Const$ and thus we need to consider
the enlarged form
\er{khjhhkfgjfjhgghhgjghjhjkkkkgjghghuiiiulkkjlkklKKgfgjhjjghgjhhjhjhhhhhghhgtyytojjjjjuyiyuughgfgint}
instead of
\er{khjhhkfgjfjhgghhgjghjhjkkkkgjghghuiiiulkkjlkklKKgfgjhjjghgjhhjhjhhhhhghhgtyytojjjjjuyiyuuint}.
Moreover, the density of the Lagrangian of the electromagnetic field
in every curvilinear coordinate system in $\mathbb{R}^4$ also has a
form of \er{vhfffngghkjgghPPNggjgjjkgjoiuioigghint}:
\begin{multline}\label{vhfffngghkjgghPPNggjgjjkgjoiuioigghjhhhint}
L_1=\frac{1}{4\pi}\left(-\sum_{n=0}^{3}\sum_{k=0}^{3}\frac{1}{4}F^{nk}F_{nk}-\sum_{k=0}^{3}4\pi
j^k A_k\right)=\\
\frac{1}{4\pi}\left(-\sum_{n=0}^{3}\sum_{k=0}^{3}\sum_{m=0}^{3}\sum_{p=0}^{3}\frac{1}{4}g^{mn}g^{pk}\left(\frac{\partial
A_p}{\partial x^m}-\frac{\partial A_m}{\partial
x^p}\right)\left(\frac{\partial A_k}{\partial x^n}-\frac{\partial
A_n}{\partial x^k}\right)-\sum_{k=0}^{3}4\pi j^k A_k\right),
\end{multline}
where
\begin{equation}\label{huohuioy89gjjhjffffff3478zzrrZZZhjhhjhhjjhhffGGhjjhjhhjhjint}
F_{ij}:=\frac{\partial A_j}{\partial x^i}-\frac{\partial
A_i}{\partial x^j}\quad\quad\forall\, i,j=0,1,2,3\,.
\end{equation}

Next the general Lagrangian of motion of the charged particle in the
gravitational and electromagnetic field
\er{btfffygtgyggyijhhkkSYSPNuhuygyygyggyuyyuyuykuhghgjjojint}
preserve its form in every curvilinear coordinate system:
\begin{equation}\label{btfffygtgyggyijhhkkSYSPNuhuygyygyggyuyyuyuykuhghgjjojhjfgint}
J_{\mathcal{G}}(\chi)=
\int_0^T\left\{-mc^2\;\mathcal{G}\left(\sum_{j=0}^{3}\sum_{k=0}^{3}g_{jk}\left(\chi(t)\right)\,\frac{d\chi^j}{dt}\,\frac{d\chi^k}{dt}\right)-\sum_{k=0}^{3}\sigma
A_k\left(\chi(t)\right)\,\frac{d\chi^k}{dt}\right\}dt.
\end{equation}
where $t$ is the global time, which is a scalar on the group
$\mathcal{S}$,
\begin{equation}\label{btfffjhgjghghijhhint}
\left(\chi^0(t),\chi^1(t),\chi^2(t),\chi^3(t)\right):=\left(\frac{1}{c}x^0(t),\frac{1}{c}x_1(t),\frac{1}{c}x_2(t),\frac{1}{c}x_3(t)\right),
\end{equation}
and $\left(x^0(t),x^1(t),x^2(t),x^3(t)\right)\in\mathbb{R}^4$ is a
four-dimensional space-time trajectory of the particle,
parameterized by the global time.

Note that if we denote by $t$ the scalar of global time, then in a
general curvilinear coordinate system the coordinate $x^0$ can
differ from $ct$, and the equality $x^0=ct$ valid, in general,
\underline{only} in \underline{cartesian} inertial or non-inertial
coordinate systems. However, since the equality in
\er{fgjfjhgghhgjghjhjijhojihjhjjijhjjjjjuiijjjkhjhjhjuiiuuuyuint}
has a covariant form, the scalar of the global time $t$ satisfies
the following Eikonal-type equation in every curvilinear coordinate
system:
\begin{equation}\label{fgjfjhgghhgjghjhjijhojihjhjjijhjjjjjuiijjjkhjhjhjuiiuuuyuuoiuuiioint}
\sum_{j=0}^{3}\sum_{k=0}^{3}\,g^{jk}
\,\frac{\partial
t}{\partial x^j}
\,\frac{\partial t}{\partial
x^k}
\,=\, \frac{1}{c^2}.
\end{equation}
Moreover, since the equality in
\er{fgjfjhgghhgjghjhjijhojihjhjjijhjjjjjuiijjjkhjhjhjuiiuuuyuint}
also has a covariant form, the following identity is valid in every
curvilinear coordinate system:
\begin{equation}\label{fgjfjhfgfgfgfhihhjhhjjhhjhjint}
\sum_{k=0}^{3}{\Theta}^{mk}\frac{\partial t}{\partial
x^k}=0\quad\quad\forall\, m=0,1,2,3,
\end{equation}
where $\Theta^{ij}$ is the contravariant tensor of the
three-dimensional geometry, that has the form
\er{huohuioy89gjjhjffffff3478zzrrZZZhjhhjhhjjhhffGGhjjhuiiuihhjkkoioiint}
\underline{only} in \underline{cartesian} inertial or non-inertial
coordinate systems.

Next, in the particular case of the relativistic-like Lagrangian
where $\mathcal{G}(\tau):=\sqrt{\tau}$, the Lagrangian in
\er{btfffygtgyggyijhhkkSYSPNuhuygyygyggyuyyuyuykuhghgjjojiyyjljlghint}
also preserve their form in every curvilinear coordinate system:
\begin{equation}\label{btfffygtgyggyijhhkkSYSPNuhuygyygyggyuyyuyuykuhghgjjojiyyjljlghhhhjhjint}
J_{rl}(\chi)=
\int_a^b\left\{-mc^2\;\sqrt{\left(\sum_{j=0}^{3}\sum_{k=0}^{3}g_{jk}\left(\chi(s)\right)\,\frac{d\chi^j}{ds}\,\frac{d\chi^k}{ds}\right)}-\sum_{k=0}^{3}\sigma
A_k\left(\chi(s)\right)\,\frac{d\chi^k}{ds}\right\}ds,
\end{equation}
where $s$ is the arbitrary parameter of the trajectory with
increasing mapping $t\leftrightarrow s$:
\begin{equation}\label{btfffjhgjghghijhhhuint}
\left(\chi^0(s),\chi^1(s),\chi^2(s),\chi^3(s)\right):=\left(\frac{1}{c}x^0(s),\frac{1}{c}x_1(s),\frac{1}{c}x_2(s),\frac{1}{c}x_3(s)\right).
\end{equation}
In particular, in the case $\frac{\partial \chi^0}{\partial t}>0$ we
can take $s:=\chi^0$ in
\er{btfffygtgyggyijhhkkSYSPNuhuygyygyggyuyyuyuykuhghgjjojiyyjljlghhhhjhjint}.

Finally, we would like to note the following fact: since in the
absence of essential gravitational masses we have $g_{ij}=J_{ij}$,
there exists a unique, up to equivalence, "preferable" inertial
coordinate system where $\vec v=\vec k=0$ everywhere. In this
particular system by
\er{hoyuiouigyfghgjh3478zzrrZZffGGjkkj88jokhjh88jkhkjhjh88int} we
have:
\begin{equation}\label{hoyuiouigyfg3478zzrrZZffGGhhjhjiuiint}
\begin{cases}
g_{00}
=1\\
g_{ij}
=-\delta_{ij}\quad\forall 1\leq i,j\leq 3\\
g_{0j}=g_{j0}
=0\quad\forall 1\leq j\leq 3.
\end{cases}
\end{equation}
and thus the Maxwell equations are the same as in the Special
Relativity. Moreover, in this system the Lagrangian of the motion of
the particle of the form
\er{btfffygtgyggyijhhkkSYSPNuhuygyygyggyuyyuyuykuhghgjjojiyyjljlghhhhjhjint}
is also the same as in the Special Relativity. Thus, since Maxwell
equations
\er{khjhhkfgjfjhgghhgjghjhjkkkkgjghghuiiiulkkjlkklKKgfgjhjjghgjhhjhjhhhhhghhgtyytojjjjjuyiyuughgfghjjhkpkint}
and the Lagrangian of the motion of particles
\er{btfffygtgyggyijhhkkSYSPNuhuygyygyggyuyyuyuykuhghgjjojiyyjljlghhhhjhjint}
preserve their form in every cartesian, non-cartesian or curvilinear
coordinate system of the group $\mathcal{S}$, they stay the same as
in Special Relativity also in the case of every cartesian,
non-cartesian or curvilinear coordinate system. Thus in the
particular case of $\mathcal{G}(\tau):=\sqrt{\tau}$ in
\er{btfffygtgyggyijhhkkSYSPNuhuygyygyggyuyyuyuykuhghgjjojhjfgint}
and in the absence of essential gravitational masses, the unique
formal mathematical difference between our model and the Special
Relativity is that in the frames of our model we consider the
Galilean Transformations as transformations of the change of
inertial cartesian coordinate systems, up to equivalence, and
\er{noninchgravortbstrjgghguittu2int} as transformations of the
change of non-inertial cartesian coordinate systems, however the
Lorenz transformations lead to inertial \underline{non-cartesian}
coordinate systems (see the following Definition
\ref{gjgjkfgkjfjkfjfjdfint}). In contrast, in the Special Relativity
the fundamental role of the Lorenz transformations, i.e. the
transformations that preserve the form
\er{hoyuiouigyfg3478zzrrZZffGGhhjhjiuiint} of the pseudo-metric
tensor, is postulated as the role of transformations of the change
of inertial cartesian coordinate systems, and at the same time the
Galilean Transformations lead to inertial \underline{non-cartesian}
coordinate systems, and transformations
\er{noninchgravortbstrjgghguittu2int} lead to non-inertial
\underline{non-cartesian} coordinate systems.

\subsubsection{Some general covariant identities in non-cartesian or curvilinear coordinate systems}
Consider the contravariant kinematic pseudo-metric tensor of inertia
$\{J^{ij}\}_{0\leq i,j\leq 3}$ on the group $\mathcal{S}$, defined
by
\begin{equation}\label{hoyuiouigyfghgjh3478zzrrZZffGGjkkjojj88jkhkhjhint}
J^{ij}:=k^ik^j-{\Theta}^{ij}\quad\quad\forall\,i,j=0,1,2,3,
\end{equation}
where $\{{\Theta}^{ij}\}_{0\leq i,j\leq 3}$ is the contravariant
tensor of the three-dimensional geometry on the group $\mathcal{S}$
(the same as in \er{fgjfjhfgfgfgfhihhjhhjjhhjhjint}) defined in
cartesian coordinate systems by
\er{huohuioy89gjjhjffffff3478zzrrZZZhjhhjhhjjhhffGGhjjhuiiuihhjkkoioiint},
and $(k^0,k^1,k^2,k^3)$ is the four-vector of the potential of
inertia on the group $\mathcal{S}$, defined in cartesian coordinate
systems by
\er{fgjfjhgghhgjghjhjijhojihjhjjijhjjjjjuiijjjkjkkhhkkhint}. Then,
as before, the reverse to $\{J^{ij}\}_{0\leq i,j\leq 3}$ covariant
tensor on the group $\mathcal{S}$ is denoted as $\{J_{ij}\}_{0\leq
i,j\leq 3}$ and satisfies
\begin{equation}\label{fgjfjhgghhgjghjhjkkkkgjghghuiiiuujhjh88jgkgjggkint}
\sum_{m=0}^{3}J^{im}J_{mj}=\begin{cases}
1\quad\text{if}\quad i=j\\
0\quad\text{if}\quad i\neq j
\end{cases}\quad\quad\forall\, i,j=0,1,2,3
\end{equation}
(as in \er{fgjfjhgghhgjghjhjkkkkgjghghuiiiuujhjh88int}).
\begin{definition}\label{gjgjkfgkjfjkfjfjdfjhhint} We say
that a given general coordinate system ($*$) is
\underline{cartesian} if the contravariant tensor of the
three-dimensional geometry $\{{\Theta}^{ij}\}_{0\leq i,j\leq 3}$ has
the following simple form in the system ($*$) (the same as in
\er{huohuioy89gjjhjffffff3478zzrrZZZhjhhjhhjjhhffGGhjjhuiiuihhjkkoioiint}):
\begin{equation}\label{huohuioy89gjjhjffffff3478zzrrZZZhjhhjhhjjhhffGGhjjhuiiuihhjkkoioijkljljkjkjkint}
\begin{cases}
{\Theta}^{00}=0
\\
{\Theta}^{0j}={\Theta}^{j0}=0\quad\quad\forall\, j=1,2,3
\\
{\Theta}^{ij}:=\delta_{ij}\quad\quad\forall\, i,j=1,2,3,
\end{cases}
\end{equation}
and at the same time in the system ($*$) we have
\begin{equation}\label{fgjfjhgghhgjghjhjijhojihjhjjijhjjjjjuiijjjkhjhjhjuiiuuuyuuoiuuiiokljkjkjljljint}
c\frac{\partial t}{\partial x^0}(x^0,x^1,x^2,x^3)=1,
\end{equation}
where $t$ is the scalar of global time (the same as in
\er{fgjfjhfgfgfgfhihhjhhjjhhjhjint}). In particular, by
\er{fgjfjhfgfgfgfhihhjhhjjhhjhjint},
\er{huohuioy89gjjhjffffff3478zzrrZZZhjhhjhhjjhhffGGhjjhuiiuihhjkkoioijkljljkjkjkint}
and
\er{fgjfjhgghhgjghjhjijhojihjhjjijhjjjjjuiijjjkhjhjhjuiiuuuyuuoiuuiiokljkjkjljljint}
all together in the system ($*$) we have:
\begin{equation}\label{fgjfjhgghhgjghjhjijhojihjhjjijhjjjjjuiijjjkhjhjhjuiiuuuyuuoiuuiiokljkjkjljjkkhhlint}
c\left(\frac{\partial t}{\partial
x^0}(x^0,x^1,x^2,x^3),\frac{\partial t}{\partial
x^1}(x^0,x^1,x^2,x^3),\frac{\partial t}{\partial
x^2}(x^0,x^1,x^2,x^3),\frac{\partial t}{\partial
x^3}(x^0,x^1,x^2,x^3)\right)=(1,0,0,0).
\end{equation}
Then, it easily can be shown that given an cartesian coordinate
system ($*$) and a general coordinate system ($**$), the system
($**$) is cartesian if and only if, up to a constant shift of the
time, the change of coordinates from system ($*$) to system ($**$)
is given by \er{noninchgravortbstrjgghguittu2intrrrZZygjyghhjint} or
equivalently by \er{noninchgravortbstrjgghguittu2intrrrZZygjygint}:
\begin{equation}\label{noninchgravortbstrjgghguittu2intrrrZZygjyghkujhkhilyiljgint}
\begin{cases}
\vec x'=A\left(\frac{x_0}{c}\right)\cdot\vec x+\vec z\left(\frac{x_0}{c}\right)=A\left(t\right)\cdot\vec x+\vec z\left(t\right),\\
x'_0=x_0=t,
\end{cases}
\end{equation}
where $A(t)\in SO(3)$ is a rotation.
\end{definition}
\begin{definition}\label{gjgjkfgkjfjkfjfjdfint} Similarly to cartesian coordinate systems, we say
that a given non-cartesian or curvilinear coordinate system ($*$) is
\underline{inertial} if the \underline{four-vector} of the potential
of inertia $(k^0,k^1,k^2,k^3)$, defined in cartesian systems by
\er{fgjfjhgghhgjghjhjijhojihjhjjijhjjjjjuiijjjkjkkhhkkhint} is
constant in $\mathbb{R}^4$. Then, it easily can be shown that given
an inertial coordinate system ($*$) and a general coordinate system
($**$), the system ($**$) is inertial if and only if the change of
coordinates from system ($*$) to system ($**$) is given by a
\underline{linaer} transformation in $\mathbb{R}^4$. In particular,
a general coordinate system ($***$) is inertial if and only if the
tensors $\{J^{ij}\}_{0\leq i,j\leq 3}$ and $\{J_{ij}\}_{0\leq
i,j\leq 3}$ are \underline{constant} in the four-dimensional
space-time in the given system ($***$). Moreover, given an inertial
cartesian coordinate system ($*$) and a general coordinate system
($**$), the system ($**$) is inertial and at the same time is
cartesian if and only if, up to a constant shift of time and, up to
equivalence of coordinate systems, the change of coordinates from
system ($*$) to system ($**$) is given by the Galilean
transformation.
\end{definition}
In the general non-cartesian or curvilinear coordinate system we
obviously have the following list of general covariant identities:
\begin{itemize}
\item
\begin{equation}\label{fgjfjhgghhgjghjhjkkkkgjghghuiiiuujhjhljhgugii88;jl;jkint}
\text{det}\,G=\text{det}\,J
\end{equation}
(see \er{fgjfjhgghhgjghjhjkkkkgjghghuiiiuujhjhljhgugii88int}), where
$4\times 4$-matrix $J$ is given by,
\begin{equation}\label{fgjfjhgghhgjghjhjkkkkgjghghuiiiuujhjhjkljojuuiyi88gkgkkkkhhkhkhint}
J=\{J_{ij}\}_{0\leq i,j\leq 3},
\end{equation}
and $4\times 4$-matrix $G$ is given by,
\begin{equation}\label{fgjfjhgghhgjghjhjkkkkgjghghuiiiuujhjhjkljojuuiyi88gkgkk;int}
G=\{g_{ij}\}_{0\leq i,j\leq 3},
\end{equation}
with the covariant pseudo-metric tensor of the four-dimensional
space-time $\{g_{ij}\}_{0\leq i,j\leq 3}$.
\item
\begin{equation}\label{hoyuiouigyfghgjh3478zzrrZZjhjhugyuyughggjhjhjyuuygtyffjjjkjljkjjljokjllllklljkljkjlnnnhkhkhkhkhh8899int}
\left\{\delta_j \left(\frac{\partial t}{\partial x_i}
\right)\right\}_J=0\quad\quad\forall\, i,j=0,1,2,3
\end{equation}
(see
\er{hoyuiouigyfghgjh3478zzrrZZjhjhugyuyughggjhjhjyuuygtyffjjjkjljkjjljokjllllklljkljkjlnnnhkhkhkhkhh88int}),
where
\begin{multline}\label{hoyuiouigyfghgjh3478zzrrZZjhjhugyuyughggjhjhjyuuygtyffjjjkjljkjjljokjllllklljkljkjlnnnljjl8899int}
\left\{\delta_j \left(\frac{\partial t}{\partial x_i}
\right)\right\}_J=\left\{\delta_i \left(\frac{\partial t}{\partial
x_j} \right)\right\}_J:=\frac{\partial^2 t}{\partial x^i\partial
x^j} -\sum_{m=0}^{3}\left\{\Gamma^{m}_{ij}\right\}_J\,\frac{\partial
t}{\partial x^m}\quad\quad\forall\, i,j=0,1,2,3\,,
\end{multline}
with
\begin{equation}\label{jhffhgffgh3478zzrrZZffjjkjkjklkkj8899int}
\begin{cases}
\left\{\Gamma_{i,mn}\right\}_J:=\frac{1}{2}\left(\frac{\partial
J_{im}}{\partial x_{n}}+\frac{\partial J_{in}}{\partial
x_{m}}-\frac{\partial
J_{mn}}{\partial x_{i}}\right)\\
\left\{\Gamma^{i}_{mn}\right\}_J:=\sum_{j=0}^{3}J^{ij}\left\{\Gamma_{j,mn}\right\}_J
\end{cases}\quad\quad\forall\, i,m,n=0,1,2,3.
\end{equation}
\item
\begin{equation}\label{fgjfjhgghhgjghjhjijhojihjhjjijhjjjjjuiijjjkhjhjhjuiiuuuyu88jukukukuk88kpjljl99int}
\sum_{m=0}^{3} \,J^{jm}\,\left(c\,\frac{\partial t}{\partial
x^m}\right)=k^j\quad\quad\forall\,j=0,1,2,3,
\end{equation}
and
\begin{equation}\label{fgjfjhgghhgjghjhjijhojihjhjjijhjjjjjuiijjjkhjhjhjuiiuuuyu88jukukukuk88kpjljlkk99int}
\sum_{m=0}^{3} \,g^{jm}\,\left(c\,\frac{\partial t}{\partial
x^m}\right)=v^j\quad\quad\forall\,j=0,1,2,3.
\end{equation}
\item
\begin{equation}\label{fgjfjhgghhgjghjhjijhojihjhjjijhjjjjjuiijjjkhjhjhjuiiuuuyu88jukukukuk88kpjljlkjojjhgjklint}
\sum_{j=0}^{3} \,k^j\,\left(c\,\frac{\partial t}{\partial
x^j}\right)=1=\sum_{j=0}^{3} \,v^j\,\left(c\,\frac{\partial
t}{\partial x^j}\right),
\end{equation}
\item
\begin{equation}\label{fgjfjhgghhgjghjhjijhojihjhjjijhjjjjjuiijjjkhjhjhjuiiuuuyu88mnmbbbjjjint}
c^2\left(\sum_{j=0}^{3}\sum_{m=0}^{3}\,J^{jm}\frac{\partial
t}{\partial x^j}\frac{\partial t}{\partial
x^m}\right)=1=c^2\left(\sum_{j=0}^{3}\sum_{m=0}^{3}\,g^{jm}\frac{\partial
t}{\partial x^j}\frac{\partial t}{\partial x^m}\right).
\end{equation}
\item
\begin{equation}\label{fgjfjhgghhgjghjhjijhojihjhjjijhjjjjjuiijjjkhjhjhjuiiuuuyu88jukukukuk88kpjljlmhmbbvnvvvkljljl88jkhkhkint}
\sum_{m=0}^{3} \,J_{jm}\,k^m\,=\,c\,\frac{\partial t}{\partial
x^j}\,=\,\sum_{m=0}^{3} \,g_{jm}\,v^m\quad\quad\forall\,j=0,1,2,3.
\end{equation}
\item
\begin{equation}\label{fgjfjhfgfgfgfhihhjhhjjhhjhjljjljkljjkint}
\sum_{m=0}^{3}{\Theta}^{jm}\frac{\partial t}{\partial
x^m}=0\quad\quad\forall\, j=0,1,2,3.
\end{equation}
\item
\begin{equation}\label{hoyuiouigyfghgjh3478zzrrZZffGGjkkjojj88jkhkhjhjljljklklint}
J^{ij}=k^ik^j-{\Theta}^{ij}\quad\quad\forall\,i,j=0,1,2,3,
\end{equation}
and
\begin{equation}\label{hoyuiouigyfghgjh3478zzrrZZffGGjkkjojj88jkhkhjhjljljjjkhhkhjk;k;int}
g^{ij}=v^iv^j-{\Theta}^{ij}\quad\quad\forall\,i,j=0,1,2,3,
\end{equation}
that implies, in particular,
\begin{equation}\label{hoyuiouigyfghgjh3478zzrrZZffGGjkkjojj88jkhkhjhjljljjjkhhkhjk;k;jkljhkint}
g^{ij}=\left(v^iv^j-k^ik^j\right)+J^{ij}\quad\quad\forall\,i,j=0,1,2,3.
\end{equation}
\item
\begin{equation}\label{fgjfjhgghhgjghjhjkkkkgjghghuiiiuujhjh88jgkgjggkjjjkint}
\sum_{m=0}^{3}J^{im}J_{mj}=\sum_{m=0}^{3}g^{im}g_{mj}=\begin{cases}
1\quad\text{if}\quad i=j\\
0\quad\text{if}\quad i\neq j
\end{cases}\quad\quad\forall\, i,j=0,1,2,3,
\end{equation}
that implies together with
\er{fgjfjhgghhgjghjhjijhojihjhjjijhjjjjjuiijjjkhjhjhjuiiuuuyu88jukukukuk88kpjljlmhmbbvnvvvkljljl88jkhkhkint}
and
\er{fgjfjhgghhgjghjhjijhojihjhjjijhjjjjjuiijjjkhjhjhjuiiuuuyu88mnmbbbjjjint}
the following:
\begin{equation}\label{fgjfjhgghhgjghjhjijhojihjhjjijhjjjjjuiijjjkhjhjhjuiiuuuyu88mnmbbbjjjoouuoojouint}
\sum_{j=0}^{3}\sum_{m=0}^{3}\,J_{jm}
k^jk^m=1=\sum_{j=0}^{3}\sum_{m=0}^{3}\,g_{jm}v^jv^m.
\end{equation}
\end{itemize}
Moreover, in the general non-cartesian or curvilinear coordinate
system, as before, we can define the Dynamical four-covector of
genuine gravity $(s_0,s_1,s_2,s_3)$ by:
\begin{equation}\label{huohuioy89gjjhjffffff3478zzrrZZZhjhhjhhjjhhffGGhjjhuiiuihhjkkoioirrkhrrjggggjrrjiokuukuppkpklhjjhj88khhjhint}
s_j=\frac{1}{2}\left(\sum_{m=0}^{3} g_{jm}k^m-\sum_{m=0}^{3}
J_{jm}v^m\right)\quad \forall j=0,1,2,3
\end{equation}
(see
\er{huohuioy89gjjhjffffff3478zzrrZZZhjhhjhhjjhhffGGhjjhuiiuihhjkkoioirrkhrrjggggjrrjiokuukuppkpklhjjhj88int}).
Then we have the following covariant identities
\begin{equation}\label{hoyuiouigyfghgjh3478zzrrZZffGGjkkjojj88iihhkhk88jklhkhjkhk8899int}
{g}_{ij}=J_{ij}+\left(s_i\left(c\,\frac{\partial t}{\partial
x^j}\right)+\left(c\,\frac{\partial t}{\partial
x^i}\right)s_j\right)\quad\quad\forall\,i,j=0,1,2,3
\end{equation}
(see
\er{hoyuiouigyfghgjh3478zzrrZZffGGjkkjojj88iihhkhk88jklhkhjkhk88int}),
and
\begin{equation}\label{hoyuiouigyfghgjh3478zzrrZZffGGjkkjojj88jklkhhk88khjhgh8899int}
v^j-k^j=\sum_{m=0}^{3}{\Theta}^{jm}s_m=\sum_{m=0}^{3}\left(k^jk^m-J^{jm}\right)s_m\quad\quad\forall\,j=0,1,2,3
\end{equation}
(see
\er{hoyuiouigyfghgjh3478zzrrZZffGGjkkjojj88jklkhhk88khjhgh88int}).

\begin{definition}\label{gjkfgjhfgdhdgdhint}
We say that a given general coordinate system ($*$) is
\underline{Lorentzian} if in system ($*$) kinematic pseudo-metric
tensors of inertia $\{J^{ij}\}_{0\leq i,j\leq 3}$ and
$\{J_{ij}\}_{0\leq i,j\leq 3}$ have the following simple form (the
same as in \er{hoyuiouigyfghgjh3478zzrrZZffGGjkkj88jokhjh88int} and
\er{hoyuiouigyfghgjh3478zzrrZZffGGjkkj88jokhjh88jkhkjhjh88int}):
\begin{equation}\label{hoyuiouigyfghgjh3478zzrrZZffGGjkkj88jokhjh8899int}
\begin{cases}
J^{00}=1\\
J^{ij}=-\delta_{ij}\quad\forall 1\leq i,j\leq 3\\
J^{0j}=J^{j0}=0\quad\forall 1\leq j\leq 3,
\end{cases}
\end{equation}
and
\begin{equation}\label{hoyuiouigyfghgjh3478zzrrZZffGGjkkj88jokhjh88jkhkjhjh8899int}
\begin{cases}
J_{00}=1\\
J_{ij}=-\delta_{ij}\quad\forall 1\leq i,j\leq 3\\
J_{0j}=J_{j0}=0\quad\forall 1\leq j\leq 3.
\end{cases}
\end{equation}
Thus, since in every Lorentzian coordinate system $\{J^{ij}\}_{0\leq
i,j\leq 3}$ and $\{J_{ij}\}_{0\leq i,j\leq 3}$ are constant in
$\mathbb{R}^4$, as it was explained in the end of Definition
\ref{gjgjkfgkjfjkfjfjdfint} every Lorentzian coordinate system is
necessary \underline{inertial}, however it is not necessary
cartesian. In particular, a unique, up to equivalence, "preferable"
inertial cartesian coordinate system ($\{0\}$) where the
three-dimensional vector potential of inertia satisfies $\vec k=0$,
is a unique, up to equivalence, coordinate system which is cartesian
and Lorentzian \underline{simultaneously}. On the other hand, there
exists a plenty of Lorentzian coordinate systems. In particular, it
is well known that, given a Lorentzian coordinate system ($*$) and a
general coordinate system ($**$), the system ($**$) is also
Lorentzian, if and only if the change of coordinates from system
($*$) to system ($**$) is given by the usual \underline{Lorentz}
transformation, up to certain trivial transformations. In every
Lorentzian coordinate system the four-vector of the potential of
inertia $(k^1,k^2,k^3,k^4)$ is \underline{constant} which is by
\er{fgjfjhgghhgjghjhjijhojihjhjjijhjjjjjuiijjjkhjhjhjuiiuuuyu88mnmbbbjjjoouuoojouint}
and \er{hoyuiouigyfghgjh3478zzrrZZffGGjkkj88jokhjh88jkhkjhjh8899int}
always satisfies
\begin{equation}\label{hoyuiouigyfghgjh3478zzrrZZffGGjkkjojj88iihhkhk88jklhkhjkhk8899jhjhjhint}
\left(k^0\right)^2-\left(\left(k^1\right)^2+\left(k^2\right)^2+\left(k^3\right)^2\right)=1.
\end{equation}
In fact in a general Lorentzian coordinate system the four-vector
$(k^1,k^2,k^3,k^4)$ can be arbitrary constant four-vector,
satisfying
\er{hoyuiouigyfghgjh3478zzrrZZffGGjkkjojj88iihhkhk88jklhkhjkhk8899jhjhjhint}.
We remind that in the unique cartesian Lorentzian coordinate system
we have $(k^1,k^2,k^3,k^4)=(1,0,0,0)$. Next by
\er{fgjfjhgghhgjghjhjijhojihjhjjijhjjjjjuiijjjkhjhjhjuiiuuuyu88jukukukuk88kpjljlmhmbbvnvvvkljljl88jkhkhkint}
together with \er{hoyuiouigyfghgjh3478zzrrZZffGGjkkj88jokhjh8899int}
we have:
\begin{equation}\label{fgjfjhgghhgjghjhjijhojihjhjjijhjjjjjuiijjjkhjhjhjuiiuuuyu88jukukukuk88kpjljlmhmbbvnvvvkljljl88jkhkhkjkjkjkjkklhhhjhkhhhkint}
\left(\frac{\partial t}{\partial x^0},\frac{\partial t}{\partial
x^1},\frac{\partial t}{\partial x^2},\frac{\partial t}{\partial
x^3}\right) \,=\, \frac{1}{c}\left(k^0,-k^1,-k^2,-k^3\right),
\end{equation}
and so, up to a constant shift in time, in every Lorentzian
coordinate system we have
\begin{equation}\label{hoyuiouigyfghgjh3478zzrrZZffGGjkkjojj88iihhkhk88jklhkhjkhk8899jhjhjhkjhkjjjkint}
t\,=\,\frac{1}{c}\,k^0x^0-\frac{1}{c}\left(k^1x^1+k^2x^2+k^3x^3\right).
\end{equation}
Note that in the unique cartesian Lorentzian coordinate system we
have $t=\frac{x^0}{c}$.
\end{definition}

\subsubsection{Certain curvilinear coordinate system in the case of
stationary radially symmetric gravitational field and relation to
the Schwarzschild metric} Assume that for a given part of the space
in some inertial or non-inertial cartesian coordinate system $(*)$
the gravitational field is stationary and radially symmetric that
means that the vectorial gravitational potential $\vec
v=(v^1,v^2,v^3)$ is independent of time variable $t$ and having the
form
\begin{equation}\label{ggyfgdfddsint}
\vec v(\vec x)=g\left(|\vec x|\right)\,\frac{\vec x}{|\vec
x|}\quad\quad\forall\,\vec x,
\end{equation}
for some scalar function $g(s):\mathbb{R}\to\mathbb{R}$. Next let
 $\Theta(\vec
x):\mathbb{R}^3\to\mathbb{R}$ be defined as:
\begin{equation}\label{ghfgfgfdfddfddfint}
\Theta(\vec x)=\xi\left(|\vec x|\right)\quad\forall\,\vec
x,\quad\quad\text{where}\quad\quad
\frac{d\xi}{ds}(s)=\frac{g(s)}{1-\frac{g^2(s)}{c^2}}\quad\quad\forall\,s,
\end{equation}
Then, consider the change of variables in the four-dimensional
space-time $\mathbb{R}^4$:
\begin{equation}\label{giuuihjghgghjgj78zzrrZZffhhhggygghghhjjknnKKMMghghint}
\begin{cases}
x'^0=x^0+\frac{\Theta\left(\left(x^1,x^2,x^3\right)\right)}{c}\\
x'^j=x^j\quad\quad\forall j=1,2,3.
\end{cases}
\end{equation}
that transforms the cartesian coordinate system $(*)$ to the
curvilinear coordinate system $(**)$ in the four-dimensional
space-time $\mathbb{R}^4$. Then in the terms of the
three-dimensional space and one-dimensional time:
\begin{equation}\label{fgjfjhgghjhjhjjhhjhhint}
(x^0,x^1,x^2,x^3):=(ct,x_1,x_2,x_3)=\left(ct,\vec x\right),
\end{equation}
we rewrite
\er{giuuihjghgghjgj78zzrrZZffhhhggygghghhjjknnKKMMghghint} as:
\begin{equation}\label{giuuihjghgghjgj78zzrrZZffhhhggygghghhjjknnKKMMint}
\begin{cases}
t'=t+\frac{\Theta(\vec x)}{c^2}\\
\vec x'=\vec x.
\end{cases}
\end{equation}
Note again, that since the new coordinate system $(**)$ in
$\mathbb{R}^4$ is \underline{curvilinear}, the time-like coordinate
$t'$ in coordinate system $(**)$ \underline{differ} from the proper
scalar of the global time.
%
%
%
%
%
%
Next consider the contravariant pseudo-metric tensor of the
four-dimensional space-time $\{g^{ij}\}_{0\leq i,j\leq 3}$ that due
to \er{hoyuiouigyfghgjh3478zzrrZZffGGjkkjint} has the form of
\begin{equation}\label{hoyuiouigyfg3478zzrrZZffhhhggughjttnnKKMMint}
\begin{cases}
g^{00}=1\\
g^{ij}=-\delta_{ij}+\frac{v^iv^j}{c^2}\quad\forall 1\leq i,j\leq 3\\
g^{0j}=g^{j0}=\frac{v^j}{c}\quad\forall 1\leq j\leq 3,
\end{cases}
\end{equation}
in the cartesian coordinate system $(*)$. We would like to find the
form $\{g'^{ij}\}_{0\leq i,j\leq 3}$ of this tensor in the
curvilinear coordinate system $(**)$. Then by
\er{fgjfjhgghhgjghjhjkkkkggghint} we have:
\begin{equation}\label{giuuihjghg78zzrrZZffhhhgguyuyyolhjhjhjjjkttnnKKMMint}
g'^{mn}=\sum_{i=0}^{3}\sum_{j=0}^{3}
\frac{\partial x'^m}{\partial x^i}\,\frac{\partial x'^n}{\partial
x^j} \,g^{ij}
\quad\quad\forall
0\leq m,n\leq 3.
\end{equation}
Then straightforward calculations presented in subsection
\ref{ghghfhffgfggffgfg} give that $\{g'^{ij}\}_{0\leq i,j\leq 3}$
has the following form in the system $(**)$:
\begin{equation}\label{giuuihjghg78zzrrZZffhhhgguyuyyolhjhjhjjjkjjjjttkklhjhjhghgikkjnnKK11KKhjhhjhjKKMMint}
\begin{cases}
g'^{00}=
\left(1-\frac{|\vec v|^2}{c^2}\right)^{-1} ,
\\
g'^{0n}=g'^{n0} =0 \quad\quad\forall 1\leq n\leq 3,
\\
g'^{mn}=\frac{v^m}{c}\frac{v^n}{c}-\delta_{mn} \quad\quad\forall
1\leq m,n\leq 3.
\end{cases}
\end{equation}
Next we find that the \underline{covariant} pseudo-metric tensor
$\{g'_{ij}\}_{0\leq i,j\leq 3}$ in the curvilinear coordinate system
$(**)$ has the following form:
\begin{equation}\label{giuuihjghg78zzrrZZffhhhgguyuyyolhjhjhjjjkjjjjttkklhjhjhghgikkjnnKK11KKhjhhjhjKKkkkllkKKMMioiiint}
\begin{cases}
g'_{00}=
\left(1-\frac{|\vec v|^2}{c^2}\right),
\\
g'_{0n}=g'_{n0} =0 \quad\quad\forall 1\leq n\leq 3,
\\
g'_{mn}=-\left(\left(1-\frac{|\vec
v|^2}{c^2}\right)^{-1}\frac{v^m}{c}\frac{v^n}{c}+\delta_{mn}\right)
\quad\quad\forall 1\leq m,n\leq 3.
\end{cases}
\end{equation}
In particular, taking into account
\er{giuuihjghgghjgj78zzrrZZffhhhggygghghhjjknnKKMMint} and
\er{ggyfgdfddsint} we deduce that the quadratic form, induced by the
covariant form of the pseudo-metric tensor $\{g'_{ij}\}_{0\leq
i,j\leq 3}$ in the curvilinear coordinate system $(**)$, that
defined on the tangent vectors
$\left(dx'^0,dx'^1,dx'^2,dx'^3\right)\in\mathbb{R}^4$ where $d\vec
x':=(dx'^1,dx'^2,dx'^3)$ has the following form:
%
%
%
%
%
%
\begin{multline}\label{ghfgfgdfdfddffdhhjhjint}
\sum_{i=0}^{3}\sum_{j=0}^{3}g'_{ij}dx'^idx'^j=\\
\left(1-\frac{\left|\vec v(\vec
x')\right|^2}{c^2}\right)dx'^2_0-\left(\left(1-\frac{\left|\vec
v(\vec x')\right|^2}{c^2}\right)^{-1}\left|\frac{\vec x'}{|\vec
x'|}\cdot d\vec x'\right|^2+\left(|d\vec x'|^2-\left|\frac{\vec
x'}{|\vec x'|}\cdot d\vec x'\right|^2\right)\right).
\end{multline}

Next, up to the end of this subsection, assume that our cartesian
coordinate system $(*)$ is non-rotating and our gravitational field
is formed by the spherical symmetric massive body of mass $m_0$ and
radius $R_0$ like the Earth, the Sun et.al. with the center at the
point $0$. Then, as we get either in
\er{MaxVacFull1ninshtrgravortghhghgjkgghklhjgkghghjjkjhjkkggjkhjkhjjhhfhjhklkhkhjjklzzzyyyhjggjhgghhjhNWNWBWHWPPN222kkkgtytghjjhpoppkkkhhhkhjhjint}
and
\er{MaxVacFull1ninshtrgravortghhghgjkgghklhjgkghghjjkjhjkkggjkhjkhjjhhfhjhklkhkhjjklzzzyyyhjggjhgghhjhNWNWBWHWPPN222kkkgtytghjjhpoppkkkhhhkhjhjiuuint}
(see Remark \ref{gygygygyggjh}) in the case of the Newtonian
gravity, or in Remark \ref{ghjfhgfdhdnw22mm} in the case of the more
general gravity model, given by \er{guigjgjffghguygjyfDDnw22mmint}
with arbitrary constant $\beta$, we have: either
\begin{equation}
\label{MaxVacFull1ninshtrgravortghhghgjkgghklhjgkghghjjkjhjkkggjkhjkhjjhhfhjhklkhkhjjklzzzyyyhjggjhgghhjhNWNWBWHWPPN222kkkgtytghjjhpoppkkkhhhkhjhjhjhjjint}
\vec v(\vec x)= \frac{\sqrt{-2\Phi_1(|\vec x|)}}{|\vec x|}\vec x,
\end{equation}
or
\begin{equation}
\label{MaxVacFull1ninshtrgravortghhghgjkgghklhjgkghghjjkjhjkkggjkhjkhjjhhfhjhklkhkhjjklzzzyyyhjggjhgghhjhNWNWBWHWPPN222kkkgtytghjjhpoppkkkhhhkhjhjiuuiuhhjint}
\vec v(\vec x)= -\frac{\sqrt{-2\Phi_1(|\vec x|)}}{|\vec x|}\vec x,
\end{equation}
where $\Phi_1$ is the classical Newtonian potential of our massive
body $m_0$ that satisfies
\begin{equation}\label{ghgffgfgfgfint}
\Phi_1(\vec x)=-\frac{Gm_0}{|\vec x|}
\end{equation}
outside of the body surface.
%
%
%
%
%
%
Both
\er{MaxVacFull1ninshtrgravortghhghgjkgghklhjgkghghjjkjhjkkggjkhjkhjjhhfhjhklkhkhjjklzzzyyyhjggjhgghhjhNWNWBWHWPPN222kkkgtytghjjhpoppkkkhhhkhjhjhjhjjint}
and
\er{MaxVacFull1ninshtrgravortghhghgjkgghklhjgkghghjjkjhjkkggjkhjkhjjhhfhjhklkhkhjjklzzzyyyhjggjhgghhjhNWNWBWHWPPN222kkkgtytghjjhpoppkkkhhhkhjhjiuuiuhhjint}
are particular cases of \er{ggyfgdfddsint}, with
\begin{equation}
\label{MaxVacFull1ninshtrgravortghhghgjkgghklhjgkghghjjkjhjkkggjkhjkhjjhhfhjhklkhkhjjklzzzyyyhjggjhgghhjhNWNWBWHWPPN222kkkgtytghjjhpoppkkkhhhkhjhjiuuiuhhjjkkkmnint}
g(s)=\pm\sqrt{-2\Phi_1(s)},
\end{equation}
and in particular, outside of the massive body surface we have:
\begin{equation}
\label{MaxVacFull1ninshtrgravortghhghgjkgghklhjgkghghjjkjhjkkggjkhjkhjjhhfhjhklkhkhjjklzzzyyyhjggjhgghhjhNWNWBWHWPPN222kkkgtytghjjhpoppkkkhhhkhjhjiuuiuhhjjkkkmnhjjhint}
g\left(|x|\right)=\pm\sqrt{\frac{2Gm_0}{|\vec x|}},
\end{equation}
Thus defining the function $\Theta(\vec x)$ as in
\er{ghfgfgfdfddfddfint}, that always can be done in the case
$\frac{2Gm_0}{R_0}<c^2$, we can define the change of variables from
coordinate system $(*)$ to the curvilinear coordinate system $(**)$
in the four-dimensional space-time $\mathbb{R}^4$ as in
\er{giuuihjghgghjgj78zzrrZZffhhhggygghghhjjknnKKMMint}:
\begin{equation}\label{giuuihjghgghjgj78zzrrZZffhhhggygghghhjjknnKKMMihint}
\begin{cases}
t'=t+\frac{\Theta(\vec x)}{c^2}\\
\vec x'=\vec x.
\end{cases}
\end{equation}
Then by inserting
\er{MaxVacFull1ninshtrgravortghhghgjkgghklhjgkghghjjkjhjkkggjkhjkhjjhhfhjhklkhkhjjklzzzyyyhjggjhgghhjhNWNWBWHWPPN222kkkgtytghjjhpoppkkkhhhkhjhjhjhjjint}
or
\er{MaxVacFull1ninshtrgravortghhghgjkgghklhjgkghghjjkjhjkkggjkhjkhjjhhfhjhklkhkhjjklzzzyyyhjggjhgghhjhNWNWBWHWPPN222kkkgtytghjjhpoppkkkhhhkhjhjiuuiuhhjint}
into
\er{giuuihjghg78zzrrZZffhhhgguyuyyolhjhjhjjjkjjjjttkklhjhjhghgikkjnnKK11KKhjhhjhjKKkkkllkKKMMioiiint}
we deduce the form of the covariant pseudo-metric tensor in the
curvilinear coordinate system $(**)$:
\begin{equation}\label{giuuihjghg78zzrrZZffhhhgguyuyyolhjhjhjjjkjjjjttkklhjhjhghgikkjnnKK11KKhjhhjhjKKkkkllkKKMMioiijjjkjkint}
\begin{cases}
g'_{00}=
\left(1+\frac{2\Phi_1(|\vec x'|)}{c^2}\right),
\\
g'_{0n}=g'_{n0} =0 \quad\quad\forall 1\leq n\leq 3,
\\
g'_{mn}=\left(\left(1+\frac{2\Phi_1(|\vec
x'|)}{c^2}\right)^{-1}\frac{2\Phi_1(|\vec
x'|)}{c^2}\frac{x'_m}{|\vec x'|}\frac{x'_n}{|\vec
x'|}-\delta_{mn}\right) \quad\quad\forall 1\leq m,n\leq 3.
\end{cases}
\end{equation}
Moreover, by \er{ghfgfgdfdfddffdhhjhjint} we have:
\begin{multline}\label{ghfgfgdfdfddffdhhjhjuiuuuiint}
\sum_{i=0}^{3}\sum_{j=0}^{3}g'_{ij}dx'^idx'^j=\\
\left(1+\frac{2\Phi_1(|\vec
x'|)}{c^2}\right)dx'^2_0-\left(\left(1+\frac{2\Phi_1(|\vec
x'|)}{c^2}\right)^{-1}\left|\frac{\vec x'}{|\vec x'|}\cdot d\vec
x'\right|^2+\left(|d\vec x'|^2-\left|\frac{\vec x'}{|\vec x'|}\cdot
d\vec x'\right|^2\right)\right).
\end{multline}
In particular, outside of the massive body surface, i.e. when
$|x'|>R_0$ we rewrite
\er{giuuihjghg78zzrrZZffhhhgguyuyyolhjhjhjjjkjjjjttkklhjhjhghgikkjnnKK11KKhjhhjhjKKkkkllkKKMMioiijjjkjkint}
and \er{ghfgfgdfdfddffdhhjhjuiuuuiint} as:
\begin{equation}\label{giuuihjghg78zzrrZZffhhhgguyuyyolhjhjhjjjkjjjjttkklhjhjhghgikkjnnKK11KKhjhhjhjKKkkkllkKKMMioiijjjkjkuyuhint}
\begin{cases}
g'_{00}=
\left(1-\frac{2Gm_0}{c^2|\vec x'|}\right),
\\
g'_{0n}=g'_{n0} =0 \quad\quad\forall 1\leq n\leq 3,
\\
g'_{mn}=-\left(\left(1-\frac{2Gm_0}{c^2|\vec
x'|}\right)^{-1}\frac{2Gm_0}{c^2|\vec x'|}\frac{x'_m}{|\vec
x'|}\frac{x'_n}{|\vec x'|}+\delta_{mn}\right) \quad\quad\forall
1\leq m,n\leq 3,
\end{cases}
\end{equation}
and
\begin{multline}\label{ghfgfgdfdfddffdhhjhjuiuuuikkjint}
\sum_{i=0}^{3}\sum_{j=0}^{3}g'_{ij}dx'^idx'^j=\\
\left(1-\frac{2Gm_0}{c^2|\vec
x'|}\right)dx'^2_0-\left(\left(1-\frac{2Gm_0}{c^2|\vec
x'|}\right)^{-1}\left|\frac{\vec x'}{|\vec x'|}\cdot d\vec
x'\right|^2+\left(|d\vec x'|^2-\left|\frac{\vec x'}{|\vec x'|}\cdot
d\vec x'\right|^2\right)\right).
\end{multline}
Therefore, we get that in coordinate system $(**)$, outside of the
massive body, the covariant pseudo-metric tensor in
\er{giuuihjghg78zzrrZZffhhhgguyuyyolhjhjhjjjkjjjjttkklhjhjhghgikkjnnKK11KKhjhhjhjKKkkkllkKKMMioiijjjkjkuyuhint}
and \er{ghfgfgdfdfddffdhhjhjuiuuuikkjint} \underline{exactly} the
same as the well known Schwarzschild metric from the General
Relativity. Indeed in the spherical coordinates in $\mathbb{R}^3$ we
rewrite \er{ghfgfgdfdfddffdhhjhjuiuuuikkjint} as:
\begin{multline}\label{ghfgfgdfdfddffdhhjhjuiuuuikkjjhhjhjint}
\sum_{i=0}^{3}\sum_{j=0}^{3}g'_{ij}dx'^idx'^j=\\
\left(1-\frac{2Gm_0}{c^2r'}\right)dx'^2_0-\left(\left(1-\frac{2Gm_0}{c^2r'}\right)^{-1}\left(dr'\right)^2+(r')^2\left((d\theta')^2+\sin^2{\left(\theta'\right)}(d\varphi')^2\right)\right),
\end{multline}
and this is exactly the classical Schwarzschild metric.

In particular, if we consider the monochromatic electromagnetic wave
of frequency $\omega$ of the form $e^{i\omega t}U(\vec x)$ in the
coordinate system $(*)$, then by
\er{giuuihjghgghjgj78zzrrZZffhhhggygghghhjjknnKKMMihint} in the
coordinate system $(**)$ the form of this light is $e^{i\omega
t'}U'(\vec x')$ where $U'(\vec x')=U(\vec
x')e^{-i\omega\frac{\Theta(\vec x')}{c^2}}$, i.e the electromagnetic
wave in the coordinate system  $(**)$ is also monochromatic of the
same frequency $\omega$. Thus all the optical effects that we find
in the frames of our model coincides with the effects considered in
the frames of General Relativity for the Schwarzschild metric. In
particular, the Michelson-Morely experiment and all Sagnac-type
effects will lead to the same result in the frame of our model like
in the case of the General relativity. Moreover, since the Maxwell
equations in both models have the same tensor form, all the
electromagnetic effects, where the time does not appear explicitly
will be the same. Similarly, the curvature of the light path in the
Sun's gravitational field will be the same in both models. Finally,
in the particular case of $\mathcal{G}(\tau)=\sqrt{\tau}$ in
\er{btfffygtgyggyijhhkkSYSPNuhuygyygyggyuyyuyuykuhghgjjojhjfgint},
i.e. in the case of the relativistic-like Lagrangian of the motion
in
\er{btfffygtgyggyijhhkkSYSPNuhuygyygyggyuyyuyuykuhghgjjojiyyyuyuuyyint}
all the mechanical effects will be the same in the frame of our
model like in the case of the General relativity  for the
Schwarzschild metric, provided that the time does not appear
explicitly in this effects. In particular, the movement of the
Mercury planet in the Sun's gravitational field will be the same in
both models, provided we take into account the relativistic-like
Lagrangian of the motion as in
\er{btfffygtgyggyijhhkkSYSPNuhuygyygyggyuyyuyuykuhghgjjojiyyyuyuuyyint}.

Finally, note that the similar, solution as in
\er{ghfgfgdfdfddffdhhjhjuiuuuikkjint} or
\er{ghfgfgdfdfddffdhhjhjuiuuuikkjjhhjhjint} is valid also for the
general laws of the gravity, given by either
\er{MaxVacFull1bjkgjhjhgjaaahkjhhjzzzyyykkknnnNWBWHWPPNNewCCmmhhjjhkkkint}
or \er{guigjgjffghguygjyfDDnw22mm34int}, where in all places we take
$M_0$ instead of $m_0$ and $\Phi_1(\vec x)=-\frac{GM_0}{|\vec x|}$
instead of \er{ghgffgfgfgfint}, with $M_0$ being the total effective
\underline{gravitational} mass of the Earth (see subsection
\ref{ghjfhgfdhdnw22mm34} for the details).

\subsubsection{General Lagraingian of the
gravitational-electromagnetic field, compatible with the general
Lagrangian of the motion in
\er{btfffygtgyggyijhhkkSYSPNuhuygyygyggyuyyuyuykuhghgjjojint}: The
case of the Newtonian-type gravity}

Given known the distribution of inertial mass density of some
continuum medium $\mu:=\mu(\vec x,t)$, the field of velocities of
this medium $\vec u:=\vec u(\vec x,t)$, the charge density
$\rho:=\rho(\vec x,t)$ and the current density $\vec j:=\vec j(\vec
x,t)$ in a cartesian coordinate system, consider a general
Lagrangian density $L$ of the unified gravitational-electromagnetic
field that generalize the Lagrangian density defined by
\er{vhfffngghkjgghDDint} and is consistent with the Lagrangian of
the motion of particles of the general form
\er{btfffygtgyggyijhhkkSYSPNuhuygyygyggyuyyuyuykuhghgjjojint}:
\begin{multline}\label{vhfffngghkjgghDDmmjjint}
L\left(\vec A,\Psi,\vec v,\Phi,\vec p,\vec
x,t\right):=\frac{1}{8\pi}\left|-\nabla_{\vec
x}\Psi-\frac{1}{c}\frac{\partial\vec A}{\partial t}+\frac{1}{c}\vec
v\times curl_{\vec x}\vec A\right|^2-\frac{1}{8\pi}\left|curl_{\vec
x}\vec A\right|^2-\left(\rho\Psi-\frac{1}{c}\vec A\cdot\vec
j\right)\\-\mu c^2 \mathcal{G}\left(1-\frac{1}{c^2}\left|\vec u-\vec
v\right|^2\right)+
\frac{1}{2}\left(d_{\vec x}\vec v+\left\{d_{\vec x}\vec
v\right\}^T\right):\left(d_{\vec x}\vec p+\left\{d_{\vec x}\vec
p\right\}^T\right)-2\left(div_{\vec x}\vec v\right)\left(div_{\vec
x}\vec p\right)\\+\frac{1}{4\pi G}\left(div_{\vec x}\vec
v\right)\left(\frac{\partial \Phi}{\partial t}+\vec
v\cdot\nabla_{\vec x} \Phi\right)+\frac{1}{4\pi
G}\Phi\left(div_{\vec x}\vec v\right)^2-\frac{\Phi}{16\pi
G}\left|d_{\vec x}\vec v+\left\{d_{\vec x}\vec
v\right\}^T\right|^2+\frac{1}{8\pi G}\left|\nabla_{\vec
x}\Phi\right|^2,
\end{multline}
where $\mathcal{G}(s):\mathbb{R}\to\mathbb{R}$ is a given function,
$\Phi$ is some ancillary proper scalar field and $\vec p$ is some
ancillary proper vector field. Then $L$ is invariant under the
change of inertial or non-inertial cartesian coordinate system of
the form \er{noninchgravortbstrjgghguittu2intrrrZZygjyg}. Then
denoting the function
\begin{equation}\label{ujghjjghjhint}
g(s):=-c^2\mathcal{G}\left(1-\frac{2s}{c^2}\right)\quad\quad\forall
s
\end{equation}
we rewrite \er{vhfffngghkjgghDDmmjjint} as:
\begin{multline}\label{vhfffngghkjgghDDmmint}
L\left(\vec A,\Psi,\vec v,\Phi,\vec p,\vec
x,t\right):=\frac{1}{8\pi}\left|-\nabla_{\vec
x}\Psi-\frac{1}{c}\frac{\partial\vec A}{\partial t}+\frac{1}{c}\vec
v\times curl_{\vec x}\vec A\right|^2-\frac{1}{8\pi}\left|curl_{\vec
x}\vec A\right|^2-\left(\rho\Psi-\frac{1}{c}\vec A\cdot\vec
j\right)\\+\mu g\left(\frac{1}{2}\left|\vec u-\vec
v\right|^2\right)+
\frac{1}{2}\left(d_{\vec x}\vec v+\left\{d_{\vec x}\vec
v\right\}^T\right):\left(d_{\vec x}\vec p+\left\{d_{\vec x}\vec
p\right\}^T\right)-2\left(div_{\vec x}\vec v\right)\left(div_{\vec
x}\vec p\right)\\+\frac{1}{4\pi G}\left(div_{\vec x}\vec
v\right)\left(\frac{\partial \Phi}{\partial t}+\vec
v\cdot\nabla_{\vec x} \Phi\right)+\frac{1}{4\pi
G}\Phi\left(div_{\vec x}\vec v\right)^2-\frac{\Phi}{16\pi
G}\left|d_{\vec x}\vec v+\left\{d_{\vec x}\vec
v\right\}^T\right|^2+\frac{1}{8\pi G}\left|\nabla_{\vec
x}\Phi\right|^2.
\end{multline}
We point out two the most important choices of function
$\mathcal{G}(s)$: fully non-relativistic choice
$\mathcal{G}(s)=\frac{s}{2}$ and correspondingly
$g(s)=\left(s-\frac{c^2}{2}\right)$; and relativistic-like choice
$\mathcal{G}(s)=\sqrt{s}$ and correspondingly
$g(s):=-c^2\sqrt{1-\frac{2s}{c^2}}$. Note also that in the first
case we have $\frac{dg}{ds}(s)=1$ and in the second case
$\frac{dg}{ds}(s)=\left(1-\frac{2s}{c^2}\right)^{-\frac{1}{2}}\approx
1$, where the last equation is valid in the case where $2s\ll c^2$.

Then, as before, in subsection \ref{hughjughghffghfgffg} we obtain
that a configuration $(\vec A,\Psi,\vec v,\Phi,\vec p)$ is a
critical point of the functional
\begin{equation}\label{btfffygtgyggyDDmmint}
J=\int_0^T\int_{\mathbb{R}^3}L\left(\vec A,\Psi,\vec v,\Phi,\vec
p,\vec x,t\right)d\vec x dt.
\end{equation}
if and only if it satisfies
\begin{equation}\label{guigjgjffghguygjyfDDmmint}
\begin{cases}
curl_{\vec x}\vec H=\frac{4\pi}{c}\vec
j+\frac{1}{c}\frac{\partial\vec D}{\partial
t}\\
div_{\vec x}\vec D=4\pi\rho\\
curl_{\vec x}\vec E+\frac{1}{c}\frac{\partial\vec B}{\partial t}=0\\
div_{\vec x}\vec B=0\\
\vec E=\vec D-\frac{1}{c}\vec v\times\vec B\\
\vec H=\vec B+\frac{1}{c}\vec v\times\vec D\\
curl_{\vec x}\left(curl_{\vec x}\vec v\right)=0
\\
\frac{\partial}{\partial t}\left\{div_{\vec x}\vec v\right\}+\vec
v\cdot\nabla_{\vec x}\left(div_{\vec x}\vec
v\right)+\frac{1}{4}\left|d_{\vec x}\vec v+\left\{d_{\vec x}\vec
v\right\}^T\right|^2=-\Delta_{\vec x}\Phi\\
\left(\mu g'\left(\frac{1}{2}\left|\vec u-\vec
v\right|^2\right)(\vec u-\vec v)+\frac{1}{4\pi c}\vec D\times\vec
B\right)\\=curl_{\vec x}\left(curl_{\vec x}\vec
p\right)-\frac{1}{4\pi G}\left(\frac{\partial}{\partial
t}\left(\nabla_{\vec x}\Phi\right)-curl_{\vec x}\left(\vec
v\times\nabla_{\vec x}\Phi\right)+\left(\Delta_{\vec
x}\Phi\right)\vec v\right).
\end{cases}
\end{equation}
where, consistently with \er{guigjgjffghPPNint} we denote
\begin{equation}\label{guigjgjffghDDmmint}
\begin{cases}
\vec D:=-\nabla_{\vec x}\Psi-\frac{1}{c}\frac{\partial\vec
A}{\partial t}+\frac{1}{c}\vec
v\times curl_{\vec x}\vec A\\
\vec B:=curl_{\vec x}\vec A
\\
\vec E:=-\nabla_{\vec x}\Psi-\frac{1}{c}\frac{\partial\vec
A}{\partial t}\\
\vec H:=curl_{\vec x}\vec A+\frac{1}{c}\vec
v\times\left(-\nabla_{\vec x}\Psi-\frac{1}{c}\frac{\partial\vec
A}{\partial t}+\frac{1}{c}\vec v\times curl_{\vec x}\vec A\right).
\end{cases}
\end{equation}

Next consider the equations of the gravitational-electromagnetic
field in the form \er{guigjgjffghguygjyfDDmmint}. Then, as before,
defining the gravitational mass
\begin{equation}\label{Mjghjghhjgint}
M:=\frac{1}{4\pi G}\Delta_{\vec x}\Phi
\end{equation}
and using the continuum equation
\begin{equation}\label{MaxVacFull1ninshtrgravortghhghgjkgghklhjgkghghjjkjhjkkggjkhjkhjjhhfhjhkhjjkhjgzzzyyykkknnnNWBWHWPPNNewZZZCCmmhjhjint}
\frac{\partial\mu}{\partial t}+div_{\vec x}\left(\mu\vec u\right)=0.
\end{equation}
we rewrite \er{guigjgjffghguygjyfDDmmint} as:
\begin{equation}\label{MaxVacFull1bjkgjhjhgjaaahkjhhjzzzyyykkknnnNWBWHWPPNNewCCmmhhjjhkkkint}
\begin{cases}
curl_{\vec x}\vec H=\frac{4\pi}{c}\vec
j+\frac{1}{c}\frac{\partial\vec D}{\partial
t},\\
div_{\vec x}\vec D=4\pi\rho,\\
curl_{\vec x}\vec E+\frac{1}{c}\frac{\partial\vec B}{\partial t}=0,\\
div_{\vec x}\vec B=0,\\
\vec E=\vec D-\frac{1}{c}\vec v\times\vec B,\\
\vec H=\vec B+\frac{1}{c}\vec v\times\vec D,\\
\frac{\partial}{\partial t}\left\{div_{\vec x}\vec v\right\}+\vec
v\cdot\nabla_{\vec x}\left(div_{\vec x}\vec
v\right)+\frac{1}{4}\left|d_{\vec x}\vec v+\left\{d_{\vec x}\vec
v\right\}^T\right|^2=-4\pi GM,\\
 curl_{\vec
x}\left(curl_{\vec x}\vec v\right)=0,
\\
\frac{\partial}{\partial t}\left(M-\mu\right)+div_{\vec
x}\left\{\left(M-\mu\right)\vec v\right\}=-div_{\vec x}\left\{\mu
\left(g'\left(\frac{1}{2}\left|\vec u-\vec
v\right|^2\right)-1\right)(\vec u-\vec v)+\frac{1}{4\pi c}\vec
D\times\vec B\right\}\,.
\end{cases}
\end{equation}
Note again for the last equation in
\er{MaxVacFull1bjkgjhjhgjaaahkjhhjzzzyyykkknnnNWBWHWPPNNewCCmmhhjjhkkkint}
that: in the fully non-relativistic case we have $g'(s)=1$ and in
the relativistic-like case we have
$g'(s)=\left(1-\frac{2s}{c^2}\right)^{-\frac{1}{2}}\approx 1$, where
the last equation is valid in the case where $2s\ll c^2$.

\subsubsection{General Lagraingian of the
gravitational-electromagnetic field, compatible with the general
Lagrangian of the motion in
\er{btfffygtgyggyijhhkkSYSPNuhuygyygyggyuyyuyuykuhghgjjojint}: the
case of some possible alternative model of the
gravity}\label{ghjfhgfdhdnw22mm34int} Consider $\vec k$ to be the
vectorial potential of the inertia, which is a generally trivial
speed-like vector field, assumed to be fixed in every fixed inertial
or non-inertial cartesian coordinate system (see Definition
\ref{jjhgyftdtdint}). Given known the distribution of inertial mass
density of some continuum medium $\mu:=\mu(\vec x,t)$, the field of
velocities of this medium $\vec u:=\vec u(\vec x,t)$, the charge
density $\rho:=\rho(\vec x,t)$ and the current density $\vec j:=\vec
j(\vec x,t)$ consider a Lagrangian density $L$ defined by
\begin{multline}\label{vhfffngghkjgghDDnw22mm34mm34int}
L\left(\vec A,\Psi,\vec v,\Phi_0,\vec h,\vec x,t\right):=\\
\frac{1}{8\pi}\left|-\nabla_{\vec
x}\Psi-\frac{1}{c}\frac{\partial\vec A}{\partial t}+\frac{1}{c}\vec
v\times curl_{\vec x}\vec A\right|^2-\frac{1}{8\pi}\left|curl_{\vec
x}\vec A\right|^2-\left(\rho\Psi-\frac{1}{c}\vec A\cdot\vec j\right)
-\mu c^2 \mathcal{G}\left(1-\frac{1}{c^2}\left|\vec u-\vec
v\right|^2\right)
\\
-\frac{c^2}{8\pi G}\left|-\nabla_{\vec x}\Phi_0
-\frac{1}{c}\frac{\partial\vec h}{\partial t}
+\frac{1}{c}\vec v\times curl_{\vec x}\vec h-\frac{1}{c}\nabla_{\vec
x}\left(\vec h\cdot\vec v\right)\right|^2+\frac{c^2}{8\pi
G}\left|curl_{\vec x}\vec h\right|^2\\+\frac{c^2\beta}{4\pi
G}\left(\frac{1}{4}\left|d_{\vec x}\vec v+\left\{d_{\vec x}\vec
v\right\}^T\right|^2-\left|\Div_{\vec x}\vec v\right|^2\right) ,
\end{multline}
where
\begin{equation}\label{btfffygtgyggyDDnw22jkhk22mm34int}
\vec h=\vec v-\vec
k\quad\text{and}\quad\Phi_0=-\frac{1}{2c}\left|\vec v-\vec
k\right|^2,
\end{equation}
$\mathcal{G}(s):\mathbb{R}\to\mathbb{R}$ is a given function, and
$\beta\in \mathbb{R}$ is some constant. Then, as before, denoting
the function
\begin{equation}\label{ujghjjghjh34int}
g(s):=-c^2\mathcal{G}\left(1-\frac{2s}{c^2}\right)\quad\quad\forall
s
\end{equation}
we rewrite \er{vhfffngghkjgghDDnw22mm34mm34int} as:
\begin{multline}\label{vhfffngghkjgghDDnw22mm34int}
L\left(\vec A,\Psi,\vec v,\Phi_0,\vec h,\vec x,t\right):=\\
\frac{1}{8\pi}\left|-\nabla_{\vec
x}\Psi-\frac{1}{c}\frac{\partial\vec A}{\partial t}+\frac{1}{c}\vec
v\times curl_{\vec x}\vec A\right|^2-\frac{1}{8\pi}\left|curl_{\vec
x}\vec A\right|^2-\left(\rho\Psi-\frac{1}{c}\vec A\cdot\vec j\right)
+\mu g\left(\frac{1}{2}\left|\vec u-\vec v\right|^2\right)
\\
-\frac{c^2}{8\pi G}\left|-\nabla_{\vec x}\Phi_0
-\frac{1}{c}\frac{\partial\vec h}{\partial t}
+\frac{1}{c}\vec v\times curl_{\vec x}\vec h-\frac{1}{c}\nabla_{\vec
x}\left(\vec h\cdot\vec v\right)\right|^2+\frac{c^2}{8\pi
G}\left|curl_{\vec x}\vec h\right|^2\\+\frac{c^2\beta}{4\pi
G}\left(\frac{1}{4}\left|d_{\vec x}\vec v+\left\{d_{\vec x}\vec
v\right\}^T\right|^2-\left|\Div_{\vec x}\vec v\right|^2\right) .
\end{multline}
In other words we have:
\begin{multline}\label{vhfffngghkjgghDDnw22jkjkjkjk22gughhg22mm34int}
L\left(\vec A,\Psi,\vec v,\Phi_0,\vec h,\vec
x,t\right)=L_1\left(\vec A,\Psi,\vec v,\vec
x,t\right):=\\
\frac{1}{8\pi}\left|-\nabla_{\vec
x}\Psi-\frac{1}{c}\frac{\partial\vec A}{\partial t}+\frac{1}{c}\vec
v\times curl_{\vec x}\vec A\right|^2-\frac{1}{8\pi}\left|curl_{\vec
x}\vec A\right|^2-\left(\rho\Psi-\frac{1}{c}\vec A\cdot\vec
j\right)+\mu g\left(\frac{1}{2}\left|\vec u-\vec v\right|^2\right)
\\
-\frac{c^2}{8\pi G}\left|
\frac{1}{c}\nabla_{\vec x}\left(\frac{1}{2}\left|\vec v-\vec
k\right|^2\right)-\frac{1}{c}\frac{\partial}{\partial t}\left(\vec
v-\vec k\right)
+\frac{1}{c}\vec v\times curl_{\vec x}\left(\vec v-\vec
k\right)-\frac{1}{c}\nabla_{\vec x}\left(\left(\vec v-\vec
k\right)\cdot\vec v\right)\right|^2\\+\frac{c^2}{8\pi
G}\left|curl_{\vec x}\left(\vec v-\vec
k\right)\right|^2+\frac{c^2\beta}{4\pi
G}\left(\frac{1}{4}\left|d_{\vec x}\vec v+\left\{d_{\vec x}\vec
v\right\}^T\right|^2-\left|\Div_{\vec x}\vec v\right|^2\right).
\end{multline}
In particular, in the inertial coordinate system where $d_{\vec
x}\vec k=0$ and $\partial_t\vec k=0$ we have:
\begin{multline}\label{vhfffngghkjgghDDnw22jkjkjkjk22gughhg22hkhjhjhjh22kljlkjk22mma34int}
L\left(\vec A,\Psi,\vec v,\Phi_0,\vec h,\vec
x,t\right)=L_1\left(\vec A,\Psi,\vec v,\vec
x,t\right):=\\
\frac{1}{8\pi}\left|-\nabla_{\vec
x}\Psi-\frac{1}{c}\frac{\partial\vec A}{\partial t}+\frac{1}{c}\vec
v\times curl_{\vec x}\vec A\right|^2-\frac{1}{8\pi}\left|curl_{\vec
x}\vec A\right|^2-\left(\rho\Psi-\frac{1}{c}\vec A\cdot\vec
j\right)+\mu g\left(\frac{1}{2}\left|\vec u-\vec v\right|^2\right)
\\
-\frac{c^2}{8\pi G}\left|
-\frac{1}{c}\frac{\partial\vec v}{\partial t}
+\frac{1}{c}\vec v\times curl_{\vec x}\vec v-\frac{1}{c}\nabla_{\vec
x}\left(\frac{1}{2}\left|\vec
v\right|^2\right)\right|^2+\frac{c^2}{8\pi G}\left|curl_{\vec x}\vec
v\right|^2\\+\frac{c^2\beta}{4\pi G}\left(\frac{1}{4}\left|d_{\vec
x}\vec v+\left\{d_{\vec x}\vec v\right\}^T\right|^2-\left|\Div_{\vec
x}\vec v\right|^2\right),
\end{multline}
Note here the advantage of inertial coordinate systems, where $L$
and $L_1$ are complectly independent of the vectorial potential of
the inertia $\vec k$. Furthermore,
we rewrite
\er{vhfffngghkjgghDDnw22jkjkjkjk22gughhg22hkhjhjhjh22kljlkjk22mma34int}
as:
\begin{multline}\label{vhfffngghkjgghDDnw22jkjkjkjk22gughhg22hkhjhjhjh22kljlkjk22mm34int}
L\left(\vec A,\Psi,\vec v,\Phi_0,\vec h,\vec
x,t\right)=L_1\left(\vec A,\Psi,\vec v,\vec
x,t\right):=\\
\frac{1}{8\pi}\left|-\nabla_{\vec
x}\Psi-\frac{1}{c}\frac{\partial\vec A}{\partial t}+\frac{1}{c}\vec
v\times curl_{\vec x}\vec A\right|^2-\frac{1}{8\pi}\left|curl_{\vec
x}\vec A\right|^2-\left(\rho\Psi-\frac{1}{c}\vec A\cdot\vec
j\right)+\mu g\left(\frac{1}{2}\left|\vec u-\vec v\right|^2\right)
\\
-\frac{c^2}{8\pi G}\left|
-\frac{1}{c}\left(\frac{\partial\vec v}{\partial t}+d_{\vec x}\vec
v\cdot\vec v\right)\right|^2+\frac{c^2}{8\pi G}\left|curl_{\vec
x}\vec v\right|^2+\frac{c^2\beta}{4\pi
G}\left(\frac{1}{4}\left|d_{\vec x}\vec v+\left\{d_{\vec x}\vec
v\right\}^T\right|^2-\left|\Div_{\vec x}\vec v\right|^2\right).
\end{multline}
Then, using Proposition \ref{yghgjtgyrtrtint} by
\er{vhfffngghkjgghDDnw22mm34int},
\er{btfffygtgyggyDDnw22jkhk22mm34int} and
\er{vhfffngghkjgghDDnw22jkjkjkjk22gughhg22mm34int} we deduce that
$L$ and $L_1$ are invariant under the change of inertial or
non-inertial cartesian coordinate system, provided that $\vec h$ and
$\vec A$ are proper vector fields, $\vec v$ is a speed-like vector
field and $\Phi_0$ and $\Psi_0:=\Psi-\frac{1}{c}\vec A\cdot\vec v$
are proper scalar fields.

We point out, again, two the most important choices of function
$\mathcal{G}(s)$ in \er{vhfffngghkjgghDDnw22mm34mm34int}: fully
non-relativistic choice $\mathcal{G}(s)=\frac{s}{2}$ and
correspondingly $g(s)=\left(s-\frac{c^2}{2}\right)$; and
relativistic-like choice $\mathcal{G}(s)=\sqrt{s}$ and
correspondingly $g(s):=-c^2\sqrt{1-\frac{2s}{c^2}}$. Note also that
in the first case we have $\frac{dg}{ds}(s)=1$ and in the second
case
$\frac{dg}{ds}(s)=\left(1-\frac{2s}{c^2}\right)^{-\frac{1}{2}}\approx
1$, where the last equation is valid in the case where $2s\ll c^2$.

We investigate critical points of the functional
\begin{equation}\label{btfffygtgyggyDDnw22mm34int}
J=\int_0^T\int_{\mathbb{R}^3}L_1\left(\vec A,\Psi,\vec v,\vec
x,t\right)d\vec x dt.
\end{equation}
We denote
\begin{equation}\label{guigjgjffghDDnw22mm34int}
\begin{cases}
\Psi_0=\Psi-\frac{1}{c}\vec A\cdot\vec v\\
\vec D=-\nabla_{\vec x}\Psi-\frac{1}{c}\frac{\partial\vec
A}{\partial t}+\frac{1}{c}\vec v\times curl_{\vec x}\vec
A=-\nabla_{\vec x}\Psi_0-\frac{1}{c}\frac{\partial\vec A}{\partial
t}+\frac{1}{c}\vec
v\times curl_{\vec x}\vec A-\frac{1}{c}\nabla_{\vec x}\left(\vec A\cdot\vec v\right)\\
\vec B=curl_{\vec x}\vec A
\\
\vec E=-\nabla_{\vec x}\Psi-\frac{1}{c}\frac{\partial\vec A}{\partial t}=\vec D-\frac{1}{c}\vec v\times\vec B\\
\vec H=curl_{\vec x}\vec A+\frac{1}{c}\vec
v\times\left(-\nabla_{\vec x}\Psi-\frac{1}{c}\frac{\partial\vec
A}{\partial t}+\frac{1}{c}\vec v\times curl_{\vec x}\vec
A\right)=\vec B+\frac{1}{c}\vec v\times\vec D.
\end{cases}
\end{equation}
and
\begin{equation}\label{guigjgjffghDDnwjkhhkhkhjklh22mm34int}
\begin{cases}
\vec R=-\nabla_{\vec x}\Phi_0-\frac{1}{c}\frac{\partial\vec
h}{\partial t}
+\frac{1}{c}\vec v\times curl_{\vec x}\vec
h-\frac{1}{c}\nabla_{\vec x}\left(\vec h\cdot\vec v\right)\\
\vec Q=curl_{\vec x}\vec h,
\end{cases}
\end{equation}
where $\vec h$ is a proper vector field and $\Phi_0$ is a proper
scalar field that are given by
\er{btfffygtgyggyDDnw22jkhk22mm34int}. In other words,
\begin{equation}\label{guigjgjffghDDnwjkhhkhkhjklh22ljjljkhk22mm34int}
\begin{cases}
\vec R=\frac{1}{c}\nabla_{\vec x}\left(\frac{1}{2}\left|\vec v-\vec
k\right|^2\right)-\frac{1}{c}\frac{\partial}{\partial t}\left(\vec
v-\vec k\right)
+\frac{1}{c}\vec v\times curl_{\vec x}\left(\vec v-\vec
k\right)-\frac{1}{c}\nabla_{\vec x}\left(\left(\vec v-\vec
k\right)\cdot\vec v\right)\\
\vec Q=curl_{\vec x}(\vec v-\vec k),
\end{cases}
\end{equation}
and in inertial coordinate system where $d_{\vec x}\vec k=0$ and
$\partial_t\vec k=0$ we also have:
\begin{equation}\label{guigjgjffghDDnwjkhhkhkhjklh22ljjljkhk22hhmm34int}
\begin{cases}
\vec R=-\frac{1}{c}\left(\frac{\partial\vec v}{\partial t}+d_{\vec
x}\vec
v\cdot\vec v\right)\\
\vec Q=curl_{\vec x}\vec v.
\end{cases}
\end{equation}
Then in section \ref{ghjfhgfdhdnw22mm34} we obtain that a
configuration $\left(\vec A,\Psi,\vec v,\vec x,t\right)$ is a
critical point of the functional
\begin{equation}\label{btfffygtgyggyDDnw22mmjhkkhhkhintgjkgjfjfjhint}
J=\int_0^T\int_{\mathbb{R}^3}L_1\left(\vec A,\Psi,\vec v,\vec
x,t\right)d\vec x dt,
\end{equation}
if and only if it satisfies:
\begin{equation}\label{guigjgjffghguygjyfDDnw22mmd34int}
\begin{cases}
curl_{\vec x}\vec B=\frac{4\pi}{c}\vec
j+\frac{1}{c}\frac{\partial\vec D}{\partial t}-\frac{1}{c}curl_{\vec
x}\left(\vec v\times \vec D\right)
\\
div_{\vec x}\vec D=4\pi\rho
\\
curl_{\vec x}\vec D+\frac{1}{c}\frac{\partial\vec B}{\partial t}
-\frac{1}{c}\,curl_{\vec x}\left(\vec v\times\vec B\right)=0\\
div_{\vec x}\vec B=0\\
curl_{\vec x}\vec R+\frac{1}{c}\frac{\partial\vec Q}{\partial t}-\frac{1}{c}curl_{\vec x}\left(\vec v\times\vec Q\right)=0\\
div_{\vec x}\vec Q=0
\\
\frac{1}{c}\left(\frac{\partial}{\partial t}\left(\Div_{\vec x}\vec
v\right)+\vec v\cdot\nabla_{\vec x}\left(\Div_{\vec x}\vec
v\right)+\frac{1}{4}\left|d_{\vec x}\vec v+\left\{d_{\vec x}\vec
v\right\}^T\right|^2\right)=\frac{1}{2c}\left|\vec
Q\right|^2-\Div_{\vec x}\vec R
\\
curl_{\vec x}(curl_{\vec x}\vec v)=curl_{\vec x}\vec Q
\\
\frac{4\pi G}{c^2}\left(\mu g'\left(\frac{1}{2}\left|\vec u-\vec
v\right|^2\right)(\vec u-\vec v)+\frac{1}{4\pi c}\vec D\times\vec
B-\frac{c}{4\pi G}\vec R\times\vec Q\right)=\\
(1+\beta)\,curl_{\vec x}\vec Q
-\frac{1}{c}\left(\frac{\partial\vec R}{\partial t}-curl_{\vec
x}\left\{\vec v\times\vec R\right\}+\left(div_{\vec x}\vec
R\right)\vec v\right),
\end{cases}
\end{equation}
and by \er{guigjgjffghDDnwjkhhkhkhjklh22ljjljkhk22hhmm34int} in the
inertial frame we have:
\begin{equation}\label{guigjgjffghguygjyfDDnw22khhjhgg22mmd34int}
\begin{cases}
curl_{\vec x}\vec B=\frac{4\pi}{c}\vec
j+\frac{1}{c}\frac{\partial\vec D}{\partial t}-\frac{1}{c}curl_{\vec
x}\left(\vec v\times \vec
D\right)\\
div_{\vec x}\vec D=4\pi\rho\\
curl_{\vec x}\vec D+\frac{1}{c}\frac{\partial\vec B}{\partial t}
-\frac{1}{c}\,curl_{\vec x}\left(\vec v\times\vec B\right)=0\\
div_{\vec x}\vec B=0\\
\vec R=-\frac{1}{c}\left(\frac{\partial\vec v}{\partial t}+d_{\vec
x}\vec
v\cdot\vec v\right)\\
\vec Q=curl_{\vec x}\vec v
\\
\frac{4\pi G}{c^2}\left(\mu g'\left(\frac{1}{2}\left|\vec u-\vec
v\right|^2\right)(\vec u-\vec v)+\frac{1}{4\pi c}\vec D\times\vec
B-\frac{c}{4\pi G}\vec R\times\vec Q\right)=\\
(1+\beta)\,curl_{\vec x}\vec Q
-\frac{1}{c}\left(\frac{\partial\vec R}{\partial t}-curl_{\vec
x}\left\{\vec v\times\vec R\right\}+\left(div_{\vec x}\vec
R\right)\vec v\right).
\end{cases}
\end{equation}
Furthermore,  taking $\Div_{\vec x}$ of the both sides of the last
equality in \er{guigjgjffghguygjyfDDnw22mmd34int} and using
continuum equation $\partial_t\mu+\Div_{\vec x}\left(\mu\vec
u\right)=0$ we deduce
\begin{multline}\label{guigjgjffghguygjyfDDnw22mmz34int}
-\left(\partial_t\mu+\Div_{\vec x}\left(\mu\vec
v\right)\right)+\Div_{\vec x}\left\{\mu
\left(g'\left(\frac{1}{2}\left|\vec u-\vec
v\right|^2\right)-1\right)(\vec u-\vec v)+\frac{1}{4\pi c}\vec
D\times\vec B-\frac{c}{4\pi  G}\vec R\times\vec Q\right\}=\\
\Div_{\vec x}\left(\mu g'\left(\frac{1}{2}\left|\vec u-\vec
v\right|^2\right)(\vec u-\vec v)+\frac{1}{4\pi c}\vec D\times\vec
B-\frac{c}{4\pi  G}\vec R\times\vec Q\right)=\\
-\frac{c}{4\pi  G}\left(\frac{\partial}{\partial t}\left(\Div_{\vec
x}\vec R\right)+\Div_{\vec x}\left\{\left(div_{\vec x}\vec
R\right)\vec v\right\}\right),
\end{multline}
Therefore, considering the proper scalar quantity $Q_0$, that we
call the field mass, which satisfies
\begin{equation}\label{jtgiktiutiurtfirfmmhfffmm34int}
Q_0:=-\mu+\frac{c}{4\pi  G}\,\Div_{\vec x}\vec R,
\end{equation}
by \er{guigjgjffghguygjyfDDnw22mmz34int} we deduce
\begin{equation}\label{jtgiktiutiurtfirfmm34int}
\frac{\partial Q_0}{\partial t}+div_{\vec x}\left(Q_0\vec
v\right)=-\Div_{\vec x}\left\{\mu
\left(g'\left(\frac{1}{2}\left|\vec u-\vec
v\right|^2\right)-1\right)(\vec u-\vec v)+\frac{1}{4\pi c}\vec
D\times\vec B-\frac{c}{4\pi  G}\vec R\times\vec Q\right\}.
\end{equation}
Thus, we rewrite \er{guigjgjffghguygjyfDDnw22mmd34int} as:
\begin{equation}\label{guigjgjffghguygjyfDDnw22mm34int}
\begin{cases}
curl_{\vec x}\vec B=\frac{4\pi}{c}\left(\vec j-\rho\vec
v\right)+\frac{1}{c}\frac{\partial\vec D}{\partial
t}-\frac{1}{c}curl_{\vec x}\left(\vec v\times \vec
D\right)+\frac{1}{c}\left(div_{\vec x}\vec D\right)\vec v,
\\
div_{\vec x}\vec D=4\pi\rho,
\\
curl_{\vec x}\vec D+\frac{1}{c}\frac{\partial\vec B}{\partial t}
-\frac{1}{c}\,curl_{\vec x}\left(\vec v\times\vec B\right)=0,\\
div_{\vec x}\vec B=0,\\
curl_{\vec x}\vec R+\frac{1}{c}\frac{\partial\vec Q}{\partial t}-\frac{1}{c}curl_{\vec x}\left(\vec v\times\vec Q\right)=0,\\
div_{\vec x}\vec Q=0,
\\
\frac{4\pi G}{c^2}\left(\mu g'\left(\frac{1}{2}\left|\vec u-\vec
v\right|^2\right)(\vec u-\vec v)+\frac{1}{4\pi c}\vec D\times\vec
B-\frac{c}{4\pi G}\vec R\times\vec Q\right)=\\
(1+\beta)\,curl_{\vec x}\vec Q
-\frac{1}{c}\left(\frac{\partial\vec R}{\partial t}-curl_{\vec
x}\left\{\vec v\times\vec R\right\}+\left(div_{\vec x}\vec
R\right)\vec v\right),\\
\Div_{\vec x}\vec R=\frac{4\pi G}{c}(\mu+Q_0),
\\
\frac{1}{c}\left(\frac{\partial}{\partial t}\left(\Div_{\vec x}\vec
v\right)+\vec v\cdot\nabla_{\vec x}\left(\Div_{\vec x}\vec
v\right)+\frac{1}{4}\left|d_{\vec x}\vec v+\left\{d_{\vec x}\vec
v\right\}^T\right|^2\right)=\frac{1}{2c}\left|\vec
Q\right|^2-\Div_{\vec x}\vec R,
\\
curl_{\vec x}(curl_{\vec x}\vec v)=curl_{\vec x}\vec Q,
\\
\frac{\partial Q_0}{\partial t}+div_{\vec x}\left(Q_0\vec
v\right)=-\Div_{\vec x}\left\{\mu
\left(g'\left(\frac{1}{2}\left|\vec u-\vec
v\right|^2\right)-1\right)(\vec u-\vec v)+\frac{1}{4\pi c}\vec
D\times\vec B-\frac{c}{4\pi  G}\vec R\times\vec Q\right\},
\end{cases}
\end{equation}
and we rewrite \er{guigjgjffghguygjyfDDnw22khhjhgg22mmd34int} in the
inertial frame as:
\begin{equation}\label{guigjgjffghguygjyfDDnw22khhjhgg22mm34int}
\begin{cases}
curl_{\vec x}\vec B=\frac{4\pi}{c}\vec
j+\frac{1}{c}\frac{\partial\vec D}{\partial t}-\frac{1}{c}curl_{\vec
x}\left(\vec v\times \vec
D\right),\\
div_{\vec x}\vec D=4\pi\rho,\\
curl_{\vec x}\vec D+\frac{1}{c}\frac{\partial\vec B}{\partial t}
-\frac{1}{c}\,curl_{\vec x}\left(\vec v\times\vec B\right)=0,\\
div_{\vec x}\vec B=0,\\
\vec R=-\frac{1}{c}\left(\frac{\partial\vec v}{\partial t}+d_{\vec
x}\vec
v\cdot\vec v\right),\\
\vec Q=curl_{\vec x}\vec v,
\\
\frac{4\pi G}{c^2}\left(\mu g'\left(\frac{1}{2}\left|\vec u-\vec
v\right|^2\right)(\vec u-\vec v)+\frac{1}{4\pi c}\vec D\times\vec
B-\frac{c}{4\pi G}\vec R\times\vec Q\right)=\\
(1+\beta)\,curl_{\vec x}\vec Q
-\frac{1}{c}\left(\frac{\partial\vec R}{\partial t}-curl_{\vec
x}\left\{\vec v\times\vec R\right\}+\left(div_{\vec x}\vec
R\right)\vec v\right),\\
\Div_{\vec x}\vec R=\frac{4\pi G}{c} (\mu+Q_0),\\
\frac{\partial Q_0}{\partial t}+div_{\vec x}\left(Q_0\vec
v\right)=-\Div_{\vec x}\left\{\mu
\left(g'\left(\frac{1}{2}\left|\vec u-\vec
v\right|^2\right)-1\right)(\vec u-\vec v)+\frac{1}{4\pi c}\vec
D\times\vec B-\frac{c}{4\pi  G}\vec R\times\vec Q\right\}.
\end{cases}
\end{equation}
As before, by Proposition \ref{yghgjtgyrtrtint} we deduce that
\er{guigjgjffghguygjyfDDnw22mm34int} is invariant under the change
of inertial or non-inertial cartesian coordinate systems. Moreover,
\er{guigjgjffghguygjyfDDnw22khhjhgg22mm34int} is invariant under the
change of inertial cartesian coordinate systems. Note also that,
both, \er{guigjgjffghguygjyfDDnw22mm34int} in an arbitrary inertial
or non-inertial cartesian coordinate system and
\er{guigjgjffghguygjyfDDnw22khhjhgg22mm34int} in an arbitrary
inertial cartesian coordinate system, are complectly independent of
the vectorial potential of the inertia $\vec k$. Note again for the
last equation in \er{guigjgjffghguygjyfDDnw22mm34int} or
\er{guigjgjffghguygjyfDDnw22khhjhgg22mm34int} that: in the fully
non-relativistic case we have $g'(s)=1$ and in the relativistic-like
case we have
$g'(s)=\left(1-\frac{2s}{c^2}\right)^{-\frac{1}{2}}\approx 1$, where
the last equation is valid in the case where $2s\ll c^2$.

Finally, as before, note that in the case of large constant
$|\beta|\gg 1$ we have $\vec Q\to 0$ in
\er{guigjgjffghguygjyfDDnw22mm34int} and thus, the gravity equations
\er{guigjgjffghguygjyfDDnw22mm34int} reduce to the equations of the
Newtonian-type Gravity in the form of
\er{MaxVacFull1bjkgjhjhgjaaahkjhhjzzzyyykkknnnNWBWHWPPNNewCCmmhhjjhkkkint}.
In that case the gravity field propagates with the infinite speed.
On the other hand, in the case of vanishing constant $\beta=0$ the
form of equations for $\vec R$ and $\vec Q$ in
\er{guigjgjffghguygjyfDDnw22mm34int} is completely the same as the
form of the Maxwell equations for $\vec D$ and $\vec B$ in
\er{guigjgjffghguygjyfDDnw22mm34int}, except of the different
meaning of "charges" and "currents" in these two sets of equations.
In that case the electromagnetic and the gravity fields propagate
with the same speed. However, in the mixed case of constant
$\beta\sim 1$ the electromagnetic and the gravity fields propagate
with different finite speeds.

\subsubsection{Covariant formulation of the laws of gravity in
cartesian and curvilinear coordinate systems: The case of the
Newtonian type gravity} Next our purpose is to make the equivalent
form of the Lagrangian density of the Newtonian-lype gravity in
\er{vhfffngghkjgghDDmmint} to be covariant and valid in every
\underline{curvilinear} coordinate system.
%
%
%
%
%
%

Assume first that our coordinate system is inertial or more
generally non-inertial and cartesian. Then consider a
three-dimensional vectorial gravitational potential $\vec
v=(v^1,v^2,v^3)$ and consider the covariant pseudometric tensor
$G=\{g_{ij}\}_{0\leq i,j\leq 3}$ defined by
\er{hoyuiouigyfg3478zzrrZZffGGhhjhjint}:
\begin{equation}\label{hoyuiouigyfg3478zzrrZZffint}
\begin{cases}
g_{00}=1-\frac{|\vec v|^2}{c^2}\\
g_{ij}=-\delta_{ij}\quad\forall 1\leq i,j\leq 3\\
g_{0j}=g_{j0}=\frac{v^j}{c}\quad\forall 1\leq j\leq 3.
\end{cases}
\end{equation}
Next consider the contravariant pseudometric tensor $\tilde
G=\{g^{ij}\}_{0\leq i,j\leq 3}$ defined by
\er{hoyuiouigyfghgjh3478zzrrZZffGGjkkjint}:
\begin{equation}\label{hoyuiouigyfghgjh3478zzrrZZffint}
\begin{cases}
g^{00}=1\\
g^{ij}=-\delta_{ij}+\frac{v^iv^j}{c^2}\quad\forall 1\leq i,j\leq 3\\
g^{0j}=g^{j0}=\frac{v^j}{c}\quad\forall 1\leq j\leq 3.
\end{cases}
\end{equation}
Next consider the Christoffel Symbols:
\begin{equation}\label{jhffhgffgh3478zzrrZZffint}
\begin{cases}
\Gamma_{i,kn}:=\frac{1}{2}\left(\frac{\partial g_{ik}}{\partial
x_{n}}+\frac{\partial g_{in}}{\partial x_{k}}-\frac{\partial
g_{kn}}{\partial x_{i}}\right)\\
\Gamma^{i}_{kn}:=\sum_{j=0}^{3}g^{ij}\Gamma_{j,kn}
\end{cases}\quad\quad\forall\, i,k,n=0,1,2,3,
\end{equation}
where $\vec x=(x_1,x_2,x_3)$, $x_0=ct$ and the point in the four
dimensional space-time is denoted as $(x^0,x^1,x^2,x^3):=(ct,\vec
x)=(x_0,x_1,x_2,x_3)$. In particular, by
\er{hoyuiouigyfg3478zzrrZZffint} and by the first equation in
\er{jhffhgffgh3478zzrrZZffint} we obtain:
\begin{equation}\label{jhffhgffghhjgj3478zzrrZZffint}
\begin{cases}
\Gamma_{0,00}=-\frac{1}{2c^3}\,\frac{\partial (|\vec v|^2)}{\partial t}\\
\Gamma_{0,k0}=\Gamma_{0,0k}=-\frac{1}{2c^2}\frac{\partial (|\vec
v|^2)}{\partial
x_{k}}\quad\quad\forall\, k=1,2,3,\\
\Gamma_{0,kn}=\frac{1}{2c}\left(\frac{\partial v^{k}}{\partial
x_{n}}+\frac{\partial v^{n}}{\partial
x_{k}}\right)\quad\quad\forall\, k,n=1,2,3
\\
\Gamma_{i,00}=\frac{1}{c^2}\,\left(\frac{\partial v^{i}}{\partial
t}+\frac{1}{2}\frac{\partial (|\vec v|^2)}{\partial
x_{i}}\right)\quad\quad\forall\, i=1,2,3,
\\
\Gamma_{i,k0}=\Gamma_{i,0k}=\frac{1}{2c}\left(\frac{\partial
v^{i}}{\partial x_{k}}-\frac{\partial v^{k}}{\partial
x_{i}}\right)\quad\quad\forall\, i,k=1,2,3,
\\
\Gamma_{i,kn}=0\quad\quad\forall\, i,k,n=1,2,3.
\end{cases}
\end{equation}

Next consider the four-dimensional gravitational potential
$(v^0,v^1,v^2,v^3)$ defined by
\er{fgjfjhgghhgjghjhjijhojihjhjjijhjjjjjuiijjjkint} as:
\begin{equation}\label{fgjfjhgghhgjghjhjijhojihjhjjijhjjjjjuiijjjkhjhjhjklkkint}
(v^0,v^1,v^2,v^3):=\left(1,\frac{1}{c}\vec v\right),
\end{equation}
and the corresponding lowered four-covector field
$(v_0,v_1,v_2,v_3)$, that we called the four-covector of
gravitational potential:
\begin{equation}\label{fgjfjhgghhgjghjhjijhojihjhjjijhjjjjjuiikkjjnj1kkklint}
(v_0,v_1,v_2,v_3):=\left(1,0,0,0\right).
\end{equation}
Note again that by
\er{fgjfjhgghhgjghjhjijhojihjhjjijhjjjjjuiikkjjnj1kkklint} and
\er{fgjfjhgghhgjghjhjkkkkgjghghuiiiulkkjlkklplikklkjljghhgghkint}
the four-covector of gravitational potential is a gradient of the
global time multiplied by the constant $c$:
\begin{equation}\label{fgjfjhgghhgjghjhjijhojihjhjjijhjjjjjuiikkjjnj1kkkljkhjjhhjklkkint}
(v_0,v_1,v_2,v_3):=\left(1,0,0,0\right)=c\left(\frac{\partial
t}{\partial x^0},\frac{\partial t}{\partial x^1},\frac{\partial
t}{\partial x^2},\frac{\partial t}{\partial x^1}\right).
\end{equation}
Furthermore, let $\{{\Theta}^{ij}\}_{0\leq i,j\leq 3}$ be the
contravariant tensor of the three-dimensional geometry that
satisfies
\er{huohuioy89gjjhjffffff3478zzrrZZZhjhhjhhjjhhffGGhjjhuiiuihhjkkoioiint}
in every non-inertial cartesian coordinate system:
\begin{equation}\label{huohuioy89gjjhjffffff3478zzrrZZZhjhhjhhjjhhffGGhjjhuiiuihhjkkoioikkint}
\begin{cases}
{\Theta}^{00}=0
\\
{\Theta}^{0j}={\Theta}^{j0}=0\quad\quad\forall\, j=1,2,3
\\
{\Theta}^{ij}:=\delta_{ij}\quad\quad\forall\, i,j=1,2,3.
\end{cases}
\end{equation}
Moreover, by \er{hoyuiouigyfghgjh3478zzrrZZffGGjkkjojjint} we have:
\begin{equation}\label{hoyuiouigyfghgjh3478zzrrZZffGGjkkjojjlkklint}
g^{ij}:=v^iv^j-{\Theta}^{ij}\quad\quad\forall\,i,j=0,1,2,3,
\end{equation}
Then in subsection \ref{dvfgdgdkjg} we prove that we can write that
given a three-dimensional vectorial gravitational potential $\vec
v$, a proper scalar $\Phi$ and a proper three-dimensional vector
$\vec p$ we have:
\begin{multline}\label{hoyuiouigyfghgjh3478zzrrZZjhjhugyuyughggjhjhjyuuygtyffjjjkjljkjjljokjljjkjkopopkklklkjjkjkjjkkokkjoipio;l;l;ll;lkklkjjkkklkkkkkjjkjkouint}
\frac{1}{2}\left(d_{\vec x}\vec v+\left\{d_{\vec x}\vec
v\right\}^T\right):\left(d_{\vec x}\vec p+\left\{d_{\vec x}\vec
p\right\}^T\right)-2\left(div_{\vec x}\vec v\right)\left(div_{\vec
x}\vec p\right)\\+\frac{1}{4\pi G}\left(div_{\vec x}\vec
v\right)\left(\frac{\partial \Phi}{\partial t}+\vec
v\cdot\nabla_{\vec x} \Phi\right)+\frac{1}{4\pi
G}\Phi\left(div_{\vec x}\vec v\right)^2-\frac{\Phi}{16\pi
G}\left|d_{\vec x}\vec v+\left\{d_{\vec x}\vec
v\right\}^T\right|^2+\frac{1}{8\pi G}\left|\nabla_{\vec
x}\Phi\right|^2=\\
\sum_{m=0}^{3}\sum_{n=0}^{3}\frac{32\pi
G}{c^4}\,{\Theta}^{mn}\frac{\partial}{\partial
x^m}\left(\sum_{j=0}^{3}\sum_{k=0}^{3}g^{kj}v_kS_j\right)\frac{\partial}{\partial
x^n}\left(\sum_{j=0}^{3}\sum_{k=0}^{3}g^{kj}v_kS_j\right)\\+\left(\sum_{i=0}^{3}\sum_{j=0}^{3}g^{ij}\left(\delta_j
v_i+\delta_i
v_j\right)\right)\left(\sum_{i=0}^{3}\sum_{j=0}^{3}g^{ij}\left(\delta_j
S_i+\delta_i
S_j\right)\right)\\-\sum_{k=0}^{3}\sum_{j=0}^{3}\sum_{m=0}^{3}\sum_{n=0}^{3}g^{km}g^{jn}\left(\delta_j
S_k+\delta_k S_j\right)\left(\delta_m v_n+\delta_n v_m\right),
\end{multline}
where the four-covector field $(S_0,S_1,S_2,S_3)$ on the group
$\mathcal{S}_0$ and the corresponding lifted four-covector field
$(S^0,S^1,S^2,S^3)$ are given by
\begin{equation}\label{hoyuiouigyfghgjh3478zzrrZZjhjhhjkjjhnmnmffjjhjkjkjint}
\begin{cases}
S_0=\frac{c^2}{16\pi G}\Phi-\frac{1}{2}\vec v\cdot\vec p\\
S_j=\frac{c}{2}p_j\quad\quad\forall 1\leq j\leq
3\quad\quad\text{where}\quad(p_1,p_2,p_3):=\vec p,
\end{cases}
\end{equation}
and
\begin{equation}\label{hoyuiouigyfghgjh3478zzrrZZjhjhhjkjjhnmnmffjjkkkjgjojkjkint}
\begin{cases}
S^0=\frac{c^2}{16\pi G}\Phi\\
S^j=\frac{c^2}{16\pi G}\Phi
\frac{v^j}{c}-\frac{c}{2}p_j\quad\quad\forall 1\leq j\leq 3,
\end{cases}
\end{equation}
and where $\delta_jS_i$ and $\delta_jv_i$ mean the covariant
derivatives of the four-covectors $(S_0,S_1,S_2,S_3)$ and
$(v_0,v_1,v_2,v_3)$, which are known from the Tensor Analysis to be
two-times covariant tensors and defined by the following:
\begin{equation}\label{hoyuiouigyfghgjh3478zzrrZZjhjhugyuyughggjhjhjyuuygtyffjjint}
\begin{cases}
\delta_j S_i:=\frac{\partial S_i}{\partial
x_j}-\sum_{k=0}^{3}\Gamma^{k}_{ij}S_k\quad\quad\forall\, 0\leq
i,j\leq 3\\
\delta_j v_i:=\frac{\partial v_i}{\partial
x_j}-\sum_{k=0}^{3}\Gamma^{k}_{ij}v_k\quad\quad\forall\, 0\leq
i,j\leq 3.
\end{cases}
\end{equation}
Note here that the \underline{right} hand side of
\er{hoyuiouigyfghgjh3478zzrrZZjhjhugyuyughggjhjhjyuuygtyffjjjkjljkjjljokjljjkjkopopkklklkjjkjkjjkkokkjoipio;l;l;ll;lkklkjjkkklkkkkkjjkjkouint}
is written in a covariant form which is valid also in every
\underline{curvilinear} coordinate system.

Next, given a system of $n$ particles with inertial masses
$m_1,\ldots,m_n$, charges $\sigma_1,\ldots,\sigma_n$, places $\vec
r_1(t),\ldots,\vec r_n(t)$ and velocities $\frac{d\vec
r_1}{dt}(t),\ldots, \frac{d\vec r_n}{dt}(t)$ in the cartesian
coordinate system, the usual definitions of the charge density,
current density and the mass density of this system are the
following:
\begin{equation}\label{btfffjhgjghghijhhjhhhjjkjkkjint}
\begin{cases}
\rho(\vec
x,t):=\sum\limits_{j=1}^{n}\sigma_j\,\delta^{(3)}\left(\vec x-\vec
r_j(t)\right),\\
\vec j(\vec x,t):=\sum\limits_{j=1}^{n}\sigma_j\frac{d\vec
r_j}{dt}(t)\,\delta^{(3)}\left(\vec x-\vec r_j(t)\right),\\
\mu(\vec x,t):=\sum\limits_{j=1}^{n}m_j\,\delta^{(3)}\left(\vec
x-\vec r_j(t)\right),
\end{cases}
\end{equation}
where $\delta^{(3)}$ is the usual Dirac-delta distribution
(generalized function) in $\mathbb{\R}^3$. Then denoting
\begin{equation}\label{btfffjhgjghghijhhjhhhjjkjkkhhjint}
\left(x^0,x^1,x^2,x^3\right):=\left(t,\frac{1}{c}\vec x\right),
\end{equation}
denoting
$\left(\chi^0_j(t),\chi^1_j(t),\chi^2_j(t),\chi^3_j(t)\right)\in\mathbb{R}^4$
to be a four-dimensional space-time trajectory of the $j$-th
particle, parameterized by the global time, which is in cartesian
system defined by the following:
\begin{equation}\label{btfffjhgjghghijhhjhhhjjkjkkint}
\left(\chi^0_j(t),\chi^1_j(t),\chi^2_j(t),\chi^3_j(t)\right):=\left(t,\frac{1}{c}\vec
r_j(t)\right),
\end{equation}
denoting by $G$ the $4\times 4$-matrix $G:=\{g_{ij}\}_{0\leq i,j\leq
3}$, which satisfies $\text{det}\,G=-1$ in every cartesian
coordinate system, and denoting by $\left(j^0,j^1,j^2,j^3\right)$
the four vector of the current which is in cartesian system defined
by the following:
\begin{equation}\label{ihjkjlhint}
\left(j^0,j^1,j^2,j^3\right):=\left(\rho, \frac{1}{c}\vec
j\right)(\vec x,t),
\end{equation}
in subsection \ref{dvfgdgdkjg} we prove, that similarly to the
General Relativity, in every curvilinear coordinate system, where
$\frac{\partial t}{\partial x^0}>0$, we have:
\begin{multline}\label{btfffjhgjghghijhhjhhhjjkjkkjjhhhjkjjkjhjkllkjint}
\left(j^0,j^1,j^2,j^3\right):=\frac{1}{\sqrt{\left|\text{det}\,G\right|}}\left(\hat\rho, \frac{1}{c}\hat\rho\,\vec {\hat u}_{x^0}\right)\quad\text{where}\\
\hat\rho:=\sum\limits_{j=1}^{n}\sigma_j \,\delta^{(3)}\left(\vec
{\hat x}
-c\left(\chi^1_j,\chi^2_j,\chi^3_j\right)(\chi^0_j)\right)\\
\quad\text{is the local charge density, calculated in the
curvilinear coordinate system}.
\end{multline}
Here $\vec {\hat x}:=(cx^1,cx^2,cx^3)$ and $\vec{\hat u}_{x^0}$ is
the field of velocities of the system, calculated in a given
curvilinear coordinate system by the differentiation of the last
three coordinates of the particle: $\vec {\hat
r}:=(c\chi^1,c\chi^2,c\chi^3)$ by the coordinate $\chi^0$ that can
be considered as the local time instead of global time $t$.
Moreover, the quantity in
\er{btfffjhgjghghijhhjhhhjjkjkkjjhhhjkjjkjhjkllkjint} equals to a
four-vector, under the change of curvilinear coordinate systems.
\begin{remark}Note here that we denoted the matrix $G=\{g_{ij}\}_{0\leq
i,j\leq 3}$ by the same letter as the Gravitational Constant $G$.
However, there is no ambiguity, since in the second case $G$ is a
constant scalar and in the first case $G$ is a matrix. Moreover, we
will use the matrix notation $G=\{g_{ij}\}_{0\leq i,j\leq 3}$ only
in the expressions containing term $\text{det}\,G$.
\end{remark}
Similarly to \er{btfffjhgjghghijhhjhhhjjkjkkjjhhhjkjjkjhjkllkjint},
the following quantities, defined in every coordinate system where
$\frac{\partial t}{\partial x^0}>0$,  equals to a four-vector and a
covariant scalar respectively, under the change of curvilinear
coordinate systems:
\begin{multline}\label{btfffjhgjghghijhhjhhhjjkjkkjjhhhjkjjkjhjjkkjklkklijhjhjkmhggmint}
\frac{1}{\sqrt{\left|\text{det}\,G\right|}}\left(\hat\mu,
\frac{1}{c}\hat\mu\,\vec {\hat u}_{x^0}\right)\quad\text{and}\quad
\frac{\hat \mu}{\sqrt{\left|\text{det}\,G\right|}}
\sqrt{\left(\sum\limits_{m=0}^{3}\sum\limits_{k=0}^{3}g_{mk}\,\hat
u^m_{x^0}\,\hat u^k_{x^0}\right)}\quad\quad\quad
\quad\text{where}\\
\hat\mu:=\sum\limits_{j=1}^{n}m_j \,\delta^{(3)}\left(\vec {\hat x}
-c\left(\chi^1_j,\chi^2_j,\chi^3_j\right)(\chi^0_j)\right)\\
\quad\text{is the local mass density, calculated in the curvilinear
coordinate system},
\end{multline}
and $\left(\hat u^0,\hat u^1,\hat u^2,\hat
u^3\right)_{x^0}=\left(1,\frac{1}{c}\vec{\hat u}_{x^0}\right)$ is
the field of four dimensional velocities of the system, calculated
in a given curvilinear coordinate system by the differentiation of
the four dimensional coordinates of the particles by the first
coordinate $\chi^0$. Note here that although the quantities
$\hat\rho$ and $\hat\mu$ are not covariant scalars and $\left(\hat
u^0,\hat u^1,\hat u^2,\hat u^3\right)_{x^0}$ is not a four-vector,
the first quantity in
\er{btfffjhgjghghijhhjhhhjjkjkkjjhhhjkjjkjhjkllkjint} equals to a
four-vector and the two first quantities in
\er{btfffjhgjghghijhhjhhhjjkjkkjjhhhjkjjkjhjjkkjklkklijhjhjkmhggmint}
equal to a four-vector and a covariant scalar, under the change of
curvilinear coordinate systems. Moreover, clearly the four
dimensional speed $\left(u^0,u^1,u^2,u^3\right)_{t}$, obtained in
curvilinear coordinate system by the differentiation by the global
time $t$, instead of the first local coordinate $\chi_0$, indeed
forms a four-vector and therefore, the quantity
\begin{equation}\label{btfffjhgjghghijhhjhhhjjkjkkjjhhhjkjjkjhjjkkjklkklijhjhjkmhggmhjhjhkkkint}
\frac{\hat \mu}{\sqrt{\left|\text{det}\,G\right|}}
\left(\sum\limits_{m=0}^{3}\sum\limits_{k=0}^{3}g_{mk}\,\hat
u^m_{x^0}\,u^k_{t}\right)\quad\text{where}\;\;
\hat\mu=\sum\limits_{j=1}^{n}m_j \,\delta^{(3)}\left(\vec {\hat x}
-c\left(\chi^1_j,\chi^2_j,\chi^3_j\right)(\chi^0_j)\right)
\end{equation}
equals to a covariant scalar, under the change of curvilinear
coordinate systems.

Next we can write the density of the Lagrangian of the
electromagnetic field, defined in \er{vhfffngghkjgghPPNggjgjjkgjint}
in the equivalent form \er{vhfffngghkjgghPPNggjgjjkgjoiuioigghint},
where the right hand side is written in a covariant form which is
valid for every curvilinear coordinate system:
\begin{multline}\label{vhfffngghkjgghPPNggjgjjkgjoiuioigghghghfguuiint}
\frac{1}{4\pi}\left(\frac{1}{2}\left|\vec
D\right|^2-\frac{1}{2}\left|\vec
B\right|^2-4\pi\left(\rho\Psi-\frac{1}{c}\vec A\cdot\vec
j\right)\right)=\\
\frac{1}{4\pi}\left(-\sum_{n=0}^{3}\sum_{k=0}^{3}\sum_{m=0}^{3}\sum_{p=0}^{3}\frac{1}{4}g^{mn}g^{pk}\left(\frac{\partial
A_p}{\partial x^m}-\frac{\partial A_m}{\partial
x^p}\right)\left(\frac{\partial A_k}{\partial x^n}-\frac{\partial
A_n}{\partial x^k}\right)-\sum_{k=0}^{3}4\pi j^k A_k\right),
\end{multline}
where $(j^0,j^1,j^2,j^3)$ is the four-vector of the current that
satisfies \er{btfffjhgjghghijhhjhhhjjkjkkjjhhhjkjjkjhjkllkjint} in
every curvilinear coordinate system where $\frac{\partial
t}{\partial x^0}>0$,  and $(A_0,A_1,A_2,A_3)$ is the four-covector
of the electromagnetic potential. Then by
\er{vhfffngghkjgghPPNggjgjjkgjoiuioigghghghfguuiint},
\er{hoyuiouigyfghgjh3478zzrrZZjhjhugyuyughggjhjhjyuuygtyffjjjkjljkjjljokjljjkjkopopkklklkjjkjkjjkkokkjoipio;l;l;ll;lkklkjjkkklkkkkkjjkjkouint}
and
\er{btfffjhgjghghijhhjhhhjjkjkkjjhhhjkjjkjhjjkkjklkklijhjhjkmhggmhjhjhkkkint}
we rewrite the Lagrangian density in \er{vhfffngghkjgghDDmmint} in
the cartesian coordinate system as:
\begin{multline}\label{vhfffngghkjgghDDmmioiuiujhhjhjint}
\frac{1}{8\pi}\left|-\nabla_{\vec
x}\Psi-\frac{1}{c}\frac{\partial\vec A}{\partial t}+\frac{1}{c}\vec
v\times curl_{\vec x}\vec A\right|^2-\frac{1}{8\pi}\left|curl_{\vec
x}\vec A\right|^2-\left(\rho\Psi-\frac{1}{c}\vec A\cdot\vec
j\right)\\+\mu g\left(\frac{1}{2}\left|\vec u-\vec
v\right|^2\right)+
\frac{1}{2}\left(d_{\vec x}\vec v+\left\{d_{\vec x}\vec
v\right\}^T\right):\left(d_{\vec x}\vec p+\left\{d_{\vec x}\vec
p\right\}^T\right)-2\left(div_{\vec x}\vec v\right)\left(div_{\vec
x}\vec p\right)\\+\frac{1}{4\pi G}\left(div_{\vec x}\vec
v\right)\left(\frac{\partial \Phi}{\partial t}+\vec
v\cdot\nabla_{\vec x} \Phi\right)+\frac{1}{4\pi
G}\Phi\left(div_{\vec x}\vec v\right)^2-\frac{\Phi}{16\pi
G}\left|d_{\vec x}\vec v+\left\{d_{\vec x}\vec
v\right\}^T\right|^2+\frac{1}{8\pi G}\left|\nabla_{\vec
x}\Phi\right|^2=L_{ge}=
\\
\frac{1}{4\pi}\left(-\sum_{n=0}^{3}\sum_{k=0}^{3}\sum_{m=0}^{3}\sum_{p=0}^{3}\frac{1}{4}g^{mn}g^{pk}\left(\frac{\partial
A_p}{\partial x^m}-\frac{\partial A_m}{\partial
x^p}\right)\left(\frac{\partial A_k}{\partial x^n}-\frac{\partial
A_n}{\partial x^k}\right)-\sum_{k=0}^{3}4\pi j^k A_k\right)\\
+\frac{\hat \mu}{\sqrt{\left|\text{det}\,G\right|}}
\left(\sum\limits_{m=0}^{3}\sum\limits_{k=0}^{3}g_{mk}\,\hat
u^m_{x^0}\,u^k_{t}\right)\left(\sum\limits_{m=0}^{3}\sum\limits_{k=0}^{3}g_{mk}\,
u^m_{t}\,u^k_{t}\right)^{-1}g\left(\frac{c^2}{2}-\sum\limits_{m=0}^{3}\sum\limits_{k=0}^{3}\frac{c^2}{2}g_{mk}\,
u^m_{t}\,u^k_{t}\right)
\\
+\sum_{m=0}^{3}\sum_{n=0}^{3}\frac{32\pi
G}{c^4}\,{\Theta}^{mn}\frac{\partial}{\partial
x^m}\left(\sum_{j=0}^{3}\sum_{k=0}^{3}g^{kj}v_kS_j\right)\frac{\partial}{\partial
x^n}\left(\sum_{j=0}^{3}\sum_{k=0}^{3}g^{kj}v_kS_j\right)\\+\left(\sum_{i=0}^{3}\sum_{j=0}^{3}g^{ij}\left(\delta_j
v_i+\delta_i
v_j\right)\right)\left(\sum_{i=0}^{3}\sum_{j=0}^{3}g^{ij}\left(\delta_j
S_i+\delta_i
S_j\right)\right)\\-\sum_{k=0}^{3}\sum_{j=0}^{3}\sum_{m=0}^{3}\sum_{n=0}^{3}g^{km}g^{jn}\left(\delta_j
S_k+\delta_k S_j\right)\left(\delta_m v_n+\delta_n v_m\right),
\end{multline}
where the right hand side is written in the covariant form and the
second equality is valid in every curvilinear coordinate system with
$\frac{\partial t}{\partial x^0}>0$, $(j^0,j^1,j^2,j^3)$ is the
four-vector of the current, that satisfies
\er{btfffjhgjghghijhhjhhhjjkjkkjjhhhjkjjkjhjkllkjint}, $\hat\mu$ is
given by
\er{btfffjhgjghghijhhjhhhjjkjkkjjhhhjkjjkjhjjkkjklkklijhjhjkmhggmhjhjhkkkint}
and in a cartesian coordinate system we have:
\begin{equation}\label{hoyuiouigyfghgjh3478zzrrZZjhjhhjkjjhnmnmffjjhjkjkjjhhjint}
\begin{cases}
S_0=\frac{c^2}{16\pi G}\Phi-\frac{1}{2}\vec v\cdot\vec p\\
S_j=\frac{c}{2}p_j\quad\quad\forall 1\leq j\leq 3.
\end{cases}
\end{equation}
and
\begin{equation}\label{hoyuiouigyfghgjh3478zzrrZZjhjhhjkjjhnmnmffjjkkkjgjojkjkkjjhint}
\begin{cases}
S^0=\frac{c^2}{16\pi G}\Phi\\
S^j=\frac{c^2}{16\pi G}\Phi
\frac{v^j}{c}-\frac{c}{2}p_j\quad\quad\forall 1\leq j\leq 3.
\end{cases}
\end{equation}
Moreover in the particular case of the relativistic-like choice
$g(s):=-c^2\sqrt{1-\frac{2s}{c^2}}$, by
\er{btfffjhgjghghijhhjhhhjjkjkkjjhhhjkjjkjhjjkkjklkklijhjhjkmhggmint}
we can write an alternative to
\er{vhfffngghkjgghDDmmioiuiujhhjhjint} as:
\begin{multline}\label{vhfffngghkjgghDDmmioiuiujhhjhjjkjkint}
\frac{1}{8\pi}\left|-\nabla_{\vec
x}\Psi-\frac{1}{c}\frac{\partial\vec A}{\partial t}+\frac{1}{c}\vec
v\times curl_{\vec x}\vec A\right|^2-\frac{1}{8\pi}\left|curl_{\vec
x}\vec A\right|^2-\left(\rho\Psi-\frac{1}{c}\vec A\cdot\vec
j\right)\\+\mu g\left(\frac{1}{2}\left|\vec u-\vec
v\right|^2\right)+
\frac{1}{2}\left(d_{\vec x}\vec v+\left\{d_{\vec x}\vec
v\right\}^T\right):\left(d_{\vec x}\vec p+\left\{d_{\vec x}\vec
p\right\}^T\right)-2\left(div_{\vec x}\vec v\right)\left(div_{\vec
x}\vec p\right)\\+\frac{1}{4\pi G}\left(div_{\vec x}\vec
v\right)\left(\frac{\partial \Phi}{\partial t}+\vec
v\cdot\nabla_{\vec x} \Phi\right)+\frac{1}{4\pi
G}\Phi\left(div_{\vec x}\vec v\right)^2-\frac{\Phi}{16\pi
G}\left|d_{\vec x}\vec v+\left\{d_{\vec x}\vec
v\right\}^T\right|^2+\frac{1}{8\pi G}\left|\nabla_{\vec
x}\Phi\right|^2=L_{ge}=\\
\frac{1}{4\pi}\left(-\sum_{n=0}^{3}\sum_{k=0}^{3}\sum_{m=0}^{3}\sum_{p=0}^{3}\frac{1}{4}g^{mn}g^{pk}\left(\frac{\partial
A_p}{\partial x^m}-\frac{\partial A_m}{\partial
x^p}\right)\left(\frac{\partial A_k}{\partial x^n}-\frac{\partial
A_n}{\partial x^k}\right)-\sum_{k=0}^{3}4\pi j^k A_k\right)\\
-\frac{c^2\hat \mu}{\sqrt{\left|\text{det}\,G\right|}}
\sqrt{\left(\sum\limits_{m=0}^{3}\sum\limits_{k=0}^{3}g_{mk}\,\hat
u^m_{x^0}\,\hat u^k_{x^0}\right)}
\\+\sum_{m=0}^{3}\sum_{n=0}^{3}\frac{32\pi
G}{c^4}\,{\Theta}^{mn}\frac{\partial}{\partial
x^m}\left(\sum_{j=0}^{3}\sum_{k=0}^{3}g^{kj}v_kS_j\right)\frac{\partial}{\partial
x^n}\left(\sum_{j=0}^{3}\sum_{k=0}^{3}g^{kj}v_kS_j\right)\\+\left(\sum_{i=0}^{3}\sum_{j=0}^{3}g^{ij}\left(\delta_j
v_i+\delta_i
v_j\right)\right)\left(\sum_{i=0}^{3}\sum_{j=0}^{3}g^{ij}\left(\delta_j
S_i+\delta_i
S_j\right)\right)\\-\sum_{k=0}^{3}\sum_{j=0}^{3}\sum_{m=0}^{3}\sum_{n=0}^{3}g^{km}g^{jn}\left(\delta_j
S_k+\delta_k S_j\right)\left(\delta_m v_n+\delta_n v_m\right),
\end{multline}
where the right hand side is written in the covariant form and the
second equality is valid in every curvilinear coordinate system with
$\frac{\partial t}{\partial x^0}>0$. On the other hand, in the case
of fully non-relativistic Lagrangian, where
$g(s)=\left(s-\frac{c^2}{2}\right)$, we can write an alternative to
\er{vhfffngghkjgghDDmmioiuiujhhjhjint} as:
\begin{multline}\label{vhfffngghkjgghDDmmioiuiujhhjhjklkklint}
\frac{1}{8\pi}\left|-\nabla_{\vec
x}\Psi-\frac{1}{c}\frac{\partial\vec A}{\partial t}+\frac{1}{c}\vec
v\times curl_{\vec x}\vec A\right|^2-\frac{1}{8\pi}\left|curl_{\vec
x}\vec A\right|^2-\left(\rho\Psi-\frac{1}{c}\vec A\cdot\vec
j\right)\\+\mu g\left(\frac{1}{2}\left|\vec u-\vec
v\right|^2\right)+
\frac{1}{2}\left(d_{\vec x}\vec v+\left\{d_{\vec x}\vec
v\right\}^T\right):\left(d_{\vec x}\vec p+\left\{d_{\vec x}\vec
p\right\}^T\right)-2\left(div_{\vec x}\vec v\right)\left(div_{\vec
x}\vec p\right)\\+\frac{1}{4\pi G}\left(div_{\vec x}\vec
v\right)\left(\frac{\partial \Phi}{\partial t}+\vec
v\cdot\nabla_{\vec x} \Phi\right)+\frac{1}{4\pi
G}\Phi\left(div_{\vec x}\vec v\right)^2-\frac{\Phi}{16\pi
G}\left|d_{\vec x}\vec v+\left\{d_{\vec x}\vec
v\right\}^T\right|^2+\frac{1}{8\pi G}\left|\nabla_{\vec
x}\Phi\right|^2=L_{ge}=\\
\frac{1}{4\pi}\left(-\sum_{n=0}^{3}\sum_{k=0}^{3}\sum_{m=0}^{3}\sum_{p=0}^{3}\frac{1}{4}g^{mn}g^{pk}\left(\frac{\partial
A_p}{\partial x^m}-\frac{\partial A_m}{\partial
x^p}\right)\left(\frac{\partial A_k}{\partial x^n}-\frac{\partial
A_n}{\partial x^k}\right)-\sum_{k=0}^{3}4\pi j^k A_k\right)\\
-\frac{c^2\hat \mu}{\sqrt{\left|\text{det}\,G\right|}}
\left(\sum\limits_{m=0}^{3}\sum\limits_{k=0}^{3}g_{mk}\,\hat
u^m_{x^0}\,u^k_{t}\right)
\\+\sum_{m=0}^{3}\sum_{n=0}^{3}\frac{32\pi
G}{c^4}\,{\Theta}^{mn}\frac{\partial}{\partial
x^m}\left(\sum_{j=0}^{3}\sum_{k=0}^{3}g^{kj}v_kS_j\right)\frac{\partial}{\partial
x^n}\left(\sum_{j=0}^{3}\sum_{k=0}^{3}g^{kj}v_kS_j\right)\\+\left(\sum_{i=0}^{3}\sum_{j=0}^{3}g^{ij}\left(\delta_j
v_i+\delta_i
v_j\right)\right)\left(\sum_{i=0}^{3}\sum_{j=0}^{3}g^{ij}\left(\delta_j
S_i+\delta_i
S_j\right)\right)\\-\sum_{k=0}^{3}\sum_{j=0}^{3}\sum_{m=0}^{3}\sum_{n=0}^{3}g^{km}g^{jn}\left(\delta_j
S_k+\delta_k S_j\right)\left(\delta_m v_n+\delta_n v_m\right),
\end{multline}
where again the right hand side is written in the covariant form and
the second equality is valid in every curvilinear coordinate system
with $\frac{\partial t}{\partial x^0}>0$.

Once we wrote the Lagrangian density
$L_{ge}:=L\left(S^k,v^k,A^k,x^k\right)_{k=0,\ldots 4}$ as a
covariant scalar, under the changes of curvilinear coordinate
systems, we can write a covariant Lagrangian as:
\begin{equation}\label{uuuuuuuuint}
J_{ge}(S^k,v^k,A^k):=\int_{(x^0,x^1,x^2,x^3)}L_{ge}\left(S^k,v^k,A^k,x^k\right)\,\sqrt{\left|\text{det}\,G\right|}\,dx^0dx^1dx^2dx^3.
\end{equation}
Although we need a term $\sqrt{\left|\text{det}\,G\right|}$ for the
covariance of the Lagrangian in curvilinear coordinate systems, in
the cartesian coordinate systems we always have
$\sqrt{\left|\text{det}\,G\right|}=1$.

Next note that the contravariant tensor of the three-dimensional
geometry $\Theta^{ij}$ which satisfies
\er{huohuioy89gjjhjffffff3478zzrrZZZhjhhjhhjjhhffGGhjjhuiiuihhjkkoioikkint}
in non-inertial cartesian coordinate systems and the scalar of the
global time $t$ are dependent on the geometry of the space-time only
and are fully determined in a given curvilinear coordinate system by
change of variables rules. In particular, the
four-\underline{covector} of the gravitational potential
$(v_0,v_1,v_2,v_3)$ is fully determined in the given curvilinear
coordinate system, since we have:
\begin{equation}\label{fgjfjhfgfgfgfhihhjhhjjhhjhjjhjhhjhjhjhjhjhhjhjint}
v_k=c\frac{\partial t}{\partial x^k}\quad\quad\forall\, k=0,1,2,3.
\end{equation}
Moreover, by
\er{fgjfjhgghhgjghjhjijhojihjhjjijhjjjjjuiijjjkhjhjhjuiiuuuyuuoiuuiioint}
and \er{fgjfjhfgfgfgfhihhjhhjjhhjhjint} we have the following
covariant identities which are valid in every curvilinear coordinate
system:
\begin{equation}\label{fgjfjhfgfgfgfhihhjhhjjhhjhjjhjhhjhjint}
\begin{cases}
\sum_{k=0}^{3}{\Theta}^{mk}v_k=\sum_{k=0}^{3}c\,{\Theta}^{mk}\frac{\partial
t}{\partial x^k}=0\quad\quad\forall\, m=0,1,2,3\\
\sum_{k=0}^{3}\sum_{j=0}^{3}g^{kj}v_kv_j=\sum_{k=0}^{3}\sum_{j=0}^{3}c^2g^{kj}\frac{\partial
t}{\partial x^k}\frac{\partial t}{\partial x^j}=1.
\end{cases}
\end{equation}
However the four-\underline{vector} of the gravitational potential
$(v^0,v^1,v^2,v^3)$, the contravariant pseudometric tensor
$g^{mn}=v^mv^n-\Theta^{mn}$ and thus also the covariant pseudometric
tensor $g_{mn}$ depend also on the physical properties of the
matter. Without knowing the physical properties of the matter the
four-vector of the gravitational potential can be arbitrary vector
$(v^0,v^1,v^2,v^3)$ that satisfies:
\begin{equation}\label{fgjfjhfgfgfgfhihhjhhjjhhjhjjhjhhjhjkjhint}
\sum_{k=0}^{3}v_kv^k=\sum_{k=0}^{3}c\frac{\partial t}{\partial
x^k}v^k
=1.
\end{equation}
Indeed for an arbitrary four-vector $(v^0,v^1,v^2,v^3)$ that
satisfies \er{fgjfjhfgfgfgfhihhjhhjjhhjhjjhjhhjhjkjhint}, denoting
$g^{mn}:=v^mv^n-\Theta^{mn}$, using
\er{fgjfjhfgfgfgfhihhjhhjjhhjhjjhjhhjhjint} and
\er{fgjfjhfgfgfgfhihhjhhjjhhjhjjhjhhjhjkjhint} we clearly obtain the
following consistency:
\begin{equation}\label{fgjfjhfgfgfgfhihhjhhjjhhjhjjhjhhjhjkjhjjkjkkint}
\sum_{j=0}^{3}g^{kj}v_j=\sum_{j=0}^{3}\left(v^kv^j-\Theta^{kj}\right)v_j=v^k\left(\sum_{j=0}^{3}v^jv_j\right)-\sum_{j=0}^{3}\Theta^{kj}v_j=v^k\quad\quad\forall
k=0,1,2,3.
\end{equation}
Thus we obtained that the four-vector of the gravitational potential
can be arbitrary four vector in \er{uuuuuuuuint} that satisfies the
linear constraint \er{fgjfjhfgfgfgfhihhjhhjjhhjhjjhjhhjhjkjhint}
where the four-covector $v_k$ is prescribed. So the four-vector
$(v^0,v^1,v^2,v^3)$ actually contains three independent scalar
functions similarly as in cartesian coordinate systems where we have
$v^0=1$. On the other hand, the four-vector $S^k$ contains four
independent scalar functions. Thus the Lagrangian in
\er{uuuuuuuuint} depends on seven independent scalar functions
characterizing the gravitational field and the four-vector of
electromagnetic potential, exactly as in cartesian coordinate
systems where we have four independent scalar functions that
characterize the electromagnetic field: scalar $\Psi$ and
three-dimensional vector $\vec A$ and seven independent scalar
functions that characterize the gravitational field: three are
contained in the three-dimensional vectorial gravitational potential
$\vec v$ and other four are the ancillary scalar field $\Phi$ and
the ancillary three-dimensional vector field $\vec p$.

Finally note that since our model of the Newtonian-type gravity in
the case of fully non-relativistic choice
\er{vhfffngghkjgghDDmmioiuiujhhjhjklkklint} and in the absence of
electromagnetic fields coincides with the classical Newtonian
gravitation, as a particular case, we obtained a covariant
formulation of the classical Newtonian gravity in curvilinear
coordinate systems.

\subsubsection{Covariant formulation of the laws of gravity in
cartesian and curvilinear coordinate systems: The case of some
alternative model of the
gravity}\label{gggggyffu997tg7gffttft999int}

Consider $\vec k$ to be the vectorial potential of the inertia,
which is a generally trivial speed-like vector field, assumed to be
fixed in every fixed inertial or non-inertial cartesian coordinate
system (see Definition \ref{jjhgyftdtdint}). In this subsection we
find a equivalent form of the Lagrangian density of the
gravitational-electromagnetic field in the case of the alternative
model of the gravity gravity having the general form
\er{vhfffngghkjgghDDnw22mm34int} or
\er{vhfffngghkjgghDDnw22jkjkjkjk22gughhg22mm34int}:
\begin{multline}\label{vhfffngghkjgghDDnw22mm3488int}
L\left(\vec A,\Psi,\vec v,\Phi_0,\vec h,\vec x,t\right):=\\
\frac{1}{8\pi}\left|-\nabla_{\vec
x}\Psi-\frac{1}{c}\frac{\partial\vec A}{\partial t}+\frac{1}{c}\vec
v\times curl_{\vec x}\vec A\right|^2-\frac{1}{8\pi}\left|curl_{\vec
x}\vec A\right|^2-\left(\rho\Psi-\frac{1}{c}\vec A\cdot\vec j\right)
+\mu g\left(\frac{1}{2}\left|\vec u-\vec v\right|^2\right)
\\
-\frac{c^2}{8\pi G}\left|-\nabla_{\vec x}\Phi_0
-\frac{1}{c}\frac{\partial\vec h}{\partial t}
+\frac{1}{c}\vec v\times curl_{\vec x}\vec h-\frac{1}{c}\nabla_{\vec
x}\left(\vec h\cdot\vec v\right)\right|^2+\frac{c^2}{8\pi
G}\left|curl_{\vec x}\vec h\right|^2\\+\frac{c^2\beta}{4\pi
G}\left(\frac{1}{4}\left|d_{\vec x}\vec v+\left\{d_{\vec x}\vec
v\right\}^T\right|^2-\left|\Div_{\vec x}\vec v\right|^2\right),
\end{multline}
where
\begin{equation}\label{btfffygtgyggyDDnw22jkhk22mm3488int}
\vec h=\vec v-\vec
k\quad\text{and}\quad\Phi_0=-\frac{1}{2c}\left|\vec v-\vec
k\right|^2,
\end{equation}
so that we have
\begin{multline}\label{vhfffngghkjgghDDnw22jkjkjkjk22gughhg22mm3488int}
L\left(\vec A,\Psi,\vec v,\Phi_0,\vec h,\vec
x,t\right)=L_1\left(\vec A,\Psi,\vec v,\vec
x,t\right):=\\
\frac{1}{8\pi}\left|-\nabla_{\vec
x}\Psi-\frac{1}{c}\frac{\partial\vec A}{\partial t}+\frac{1}{c}\vec
v\times curl_{\vec x}\vec A\right|^2-\frac{1}{8\pi}\left|curl_{\vec
x}\vec A\right|^2-\left(\rho\Psi-\frac{1}{c}\vec A\cdot\vec
j\right)+\mu g\left(\frac{1}{2}\left|\vec u-\vec v\right|^2\right)
\\
-\frac{c^2}{8\pi G}\left|
\frac{1}{c}\nabla_{\vec x}\left(\frac{1}{2}\left|\vec v-\vec
k\right|^2\right)-\frac{1}{c}\frac{\partial}{\partial t}\left(\vec
v-\vec k\right)
+\frac{1}{c}\vec v\times curl_{\vec x}\left(\vec v-\vec
k\right)-\frac{1}{c}\nabla_{\vec x}\left(\left(\vec v-\vec
k\right)\cdot\vec v\right)\right|^2\\+\frac{c^2}{8\pi
G}\left|curl_{\vec x}\left(\vec v-\vec
k\right)\right|^2+\frac{c^2\beta}{4\pi
G}\left(\frac{1}{4}\left|d_{\vec x}\vec v+\left\{d_{\vec x}\vec
v\right\}^T\right|^2-\left|\Div_{\vec x}\vec v\right|^2\right).
\end{multline}
Our purpose is to make the equivalent form of this Lagrangian
density to be covariant and valid in every \underline{curvilinear}
coordinate system. Assume first, that we deal with a cartesian
inertial or non-inertial coordinate system. Consider again the
three-dimensional vectorial gravitational potential $\vec
v=(v^1,v^2,v^3)$ and consider the covariant pseudometric tensor
$G=\{g_{ij}\}_{0\leq i,j\leq 3}$ defined by
\er{hoyuiouigyfg3478zzrrZZffint} and the contravariant pseudometric
tensor $\tilde G=\{g^{ij}\}_{0\leq i,j\leq 3}$ defined by
\er{hoyuiouigyfghgjh3478zzrrZZffint}. Next, as before, we consider
the four-dimensional gravitational potential $(v^0,v^1,v^2,v^3)$
defined by
\er{fgjfjhgghhgjghjhjijhojihjhjjijhjjjjjuiijjjkhjhjhjklkkint} and
the corresponding lowered four-covector field $(v_0,v_1,v_2,v_3)$,
that we called the four-covector of gravitational potential, defined
by \er{fgjfjhgghhgjghjhjijhojihjhjjijhjjjjjuiikkjjnj1kkklint}. Then,
we have
\begin{multline}\label{hoyuiouigyfghgjh3478zzrrZZjhjhugyuyughggjhjhjyuuygtyffjjjkjljkjjljokjljjkjkopopkklklkjjkjkjjkkokkjoipio;l;l;ll;lkklkjjkkklkkkkkjojjhint}
\left(\left(div_{\vec x}\vec v\right)^2-\frac{1}{4}\left|d_{\vec
x}\vec v+\left\{d_{\vec x}\vec v\right\}^T\right|^2\right)=\\
\frac{c^2}{4}\left(\left(\sum_{i=0}^{3}\sum_{j=0}^{3}g^{ij}\left(\delta_j
v_i+\delta_i
v_j\right)\right)^2-\sum_{k=0}^{3}\sum_{j=0}^{3}\sum_{m=0}^{3}\sum_{n=0}^{3}\left(\delta_j
v_k+\delta_k v_j\right)g^{km}g^{jn}\left(\delta_m v_n+\delta_n
v_m\right)\right),
\end{multline}
where $\delta_j$ means the covariant derivative with respect to the
dynamical pseudo-metrics $G=\{g_{ij}\}_{0\leq i,j\leq 3}$ (see
subsection \ref{gggggyffu997tg7gffttft999} for the details).
Furthermore, consider the Dynamical four-covector of genuine gravity
$(s_0,s_1,s_2,s_3)$, defined as in
\er{fgjfjhgghhgjghjhjijhojihjhjjijhjjjjjuiijjjk88int} by:
\begin{multline}\label{fgjfjhgghhgjghjhjijhojihjhjjijhjjjjjuiijjjk88jihhhiint}
(s_0,s_1,s_2,s_3)
=-\frac{1}{c}\left(\left(\Phi_0+\frac{1}{c}\vec v\cdot\vec
h\right),-\vec h\right)=\left(\frac{1}{2c^2}|\vec
h|^2-\frac{1}{c^2}\vec v\cdot\vec h\,,\,\frac{1}{c}\vec
h\right)\quad\quad\text{where}\quad\\ s_0
=-\frac{1}{c}\left(\Phi_0+\frac{1}{c}\vec v\cdot\vec
h\right)=\frac{1}{2c^2}|\vec h|^2-\frac{1}{c^2}\vec v\cdot\vec
h\;\;\quad\text{and}\quad\;\;(s_1,s_2,s_3)=\frac{1}{c}\vec h,
\end{multline}
where $\vec h$ and $\Phi_0$ are given by
\er{btfffygtgyggyDDnw22jkhk22mm3488int}. Next, as in
\er{MaxVacFullPPNhjjghjjkjhhiuyyint} and
\er{vhfffngghkjgghPPNggjgjjkgjoiuioigghint} we have the following:
\begin{multline}\label{vhfffngghkjgghPPNggjgjjkgjoiuijouuiyiyyyuyuint}
\frac{1}{4\pi}\left(\frac{1}{2}\left|-\nabla_{\vec
x}\Psi-\frac{1}{c}\frac{\partial\vec A}{\partial t}+\frac{1}{c}\vec
v\times curl_{\vec x}\vec A\right|^2-\frac{1}{2}\left|curl_{\vec
x}\vec A\right|^2-4\pi\left(\rho\Psi-\frac{1}{c}\vec A\cdot\vec
j\right)\right)=\\
\frac{1}{4\pi}\left(-\sum_{n=0}^{3}\sum_{k=0}^{3}\sum_{m=0}^{3}\sum_{p=0}^{3}\frac{1}{4}g^{mn}g^{pk}\left(\frac{\partial
A_p}{\partial x^m}-\frac{\partial A_m}{\partial
x^p}\right)\left(\frac{\partial A_k}{\partial x^n}-\frac{\partial
A_n}{\partial x^k}\right)-\sum_{k=0}^{3}4\pi j^k A_k\right),
\end{multline}
where $(A_0,A_1,A_2,A_3)$ is the four-covector of thefour
dimensional electromagnetic potential, defined as in
\er{fgjfjhgghhgjghjhjijhojihjhjjijhjjjljljpkint} by:
\begin{equation}\label{fgjfjhgghhgjghjhjijhojihjhjjijhjjjljljpkjkjkjojint}
(A_0,A_1,A_2,A_3)=(\Psi,-\vec A)\quad\text{where}\quad
A_0=\Psi\;\;\text{and}\;\;(A_1,A_2,A_3)=-\vec A,
\end{equation}
and $(j^0,j^1,j^2,j^3)$ is the four-vector of the current, given by
\er{btfffjhgjghghijhhjhhhjjkjkkjjhhhjkjjkjhjkllkjint}. Then,
completely analogously, as we get in
\er{MaxVacFullPPNhjjghjjkjhhiuyyint} and
\er{vhfffngghkjgghPPNggjgjjkgjoiuioigghint} the following part of
\er{vhfffngghkjgghPPNggjgjjkgjoiuijouuiyiyyyuyuint}:
\begin{multline}\label{vhfffngghkjgghPPNggjgjjkgjoiuipopiioioint}
\left(\frac{1}{2}\left|-\nabla_{\vec
x}\Psi-\frac{1}{c}\frac{\partial\vec A}{\partial t}+\frac{1}{c}\vec
v\times curl_{\vec x}\vec A\right|^2-\frac{1}{2}\left|curl_{\vec
x}\vec A\right|^2\right)=\\
-\sum_{n=0}^{3}\sum_{k=0}^{3}\sum_{m=0}^{3}\sum_{p=0}^{3}\frac{1}{4}g^{mn}g^{pk}\left(\frac{\partial
A_p}{\partial x^m}-\frac{\partial A_m}{\partial
x^p}\right)\left(\frac{\partial A_k}{\partial x^n}-\frac{\partial
A_n}{\partial x^k}\right),
\end{multline}
we can get also:
\begin{multline}\label{vhfffngghkjgghPPNggjgjjkgjoiuipopiioioouuouint}
\left(\frac{1}{2}\left|-\nabla_{\vec x}\Phi_0
-\frac{1}{c}\frac{\partial\vec h}{\partial t}
+\frac{1}{c}\vec v\times curl_{\vec x}\vec h-\frac{1}{c}\nabla_{\vec
x}\left(\vec h\cdot\vec
v\right)\right|^2-\frac{1}{2}\left|curl_{\vec
x}\vec h\right|^2\right)=\\
\left(\frac{1}{2}\left|-\nabla_{\vec x}\left(\Phi_0+\frac{1}{c}\vec
v\cdot\vec h\right) -\frac{1}{c}\frac{\partial\vec h}{\partial t}
+\frac{1}{c}\vec v\times curl_{\vec x}\vec
h\right|^2-\frac{1}{2}\left|curl_{\vec x}\vec
h\right|^2\right)=\\
-\sum_{n=0}^{3}\sum_{k=0}^{3}\sum_{m=0}^{3}\sum_{p=0}^{3}\frac{c^2}{4}g^{mn}g^{pk}\left(\frac{\partial
s_p}{\partial x^m}-\frac{\partial s_m}{\partial
x^p}\right)\left(\frac{\partial s_k}{\partial x^n}-\frac{\partial
s_n}{\partial x^k}\right).
\end{multline}
Then, similarly as it was done in
\er{vhfffngghkjgghDDmmioiuiujhhjhjint}, by
\er{vhfffngghkjgghPPNggjgjjkgjoiuijouuiyiyyyuyuint},
\er{vhfffngghkjgghPPNggjgjjkgjoiuipopiioioouuouint} and
\er{hoyuiouigyfghgjh3478zzrrZZjhjhugyuyughggjhjhjyuuygtyffjjjkjljkjjljokjljjkjkopopkklklkjjkjkjjkkokkjoipio;l;l;ll;lkklkjjkkklkkkkkjojjhint}
we rewrite the Lagrangian density in
\er{vhfffngghkjgghDDnw22mm3488int} or
\er{vhfffngghkjgghDDnw22jkjkjkjk22gughhg22mm3488int} in the
cartesian coordinate system as:
\begin{multline}\label{vhfffngghkjgghDDmmioiuiujhhjhjjhyryrrint}
\frac{1}{8\pi}\left|-\nabla_{\vec
x}\Psi-\frac{1}{c}\frac{\partial\vec A}{\partial t}+\frac{1}{c}\vec
v\times curl_{\vec x}\vec A\right|^2-\frac{1}{8\pi}\left|curl_{\vec
x}\vec A\right|^2-\left(\rho\Psi-\frac{1}{c}\vec A\cdot\vec j\right)
+\mu g\left(\frac{1}{2}\left|\vec u-\vec v\right|^2\right)
\\
-\frac{c^2}{8\pi G}\left|-\nabla_{\vec x}\Phi_0
-\frac{1}{c}\frac{\partial\vec h}{\partial t}
+\frac{1}{c}\vec v\times curl_{\vec x}\vec h-\frac{1}{c}\nabla_{\vec
x}\left(\vec h\cdot\vec v\right)\right|^2+\frac{c^2}{8\pi
G}\left|curl_{\vec x}\vec h\right|^2\\+\frac{c^2\beta}{4\pi
G}\left(\frac{1}{4}\left|d_{\vec x}\vec v+\left\{d_{\vec x}\vec
v\right\}^T\right|^2-\left|\Div_{\vec x}\vec v\right|^2\right)
=L_{ge}=
\\
\frac{1}{4\pi}\left(-\sum_{n=0}^{3}\sum_{k=0}^{3}\sum_{m=0}^{3}\sum_{p=0}^{3}\frac{1}{4}g^{mn}g^{pk}\left(\frac{\partial
A_p}{\partial x^m}-\frac{\partial A_m}{\partial
x^p}\right)\left(\frac{\partial A_k}{\partial x^n}-\frac{\partial
A_n}{\partial x^k}\right)-\sum_{k=0}^{3}4\pi j^k A_k\right)\\
+\frac{\hat \mu}{\sqrt{\left|\text{det}\,G\right|}}
\left(\sum\limits_{m=0}^{3}\sum\limits_{k=0}^{3}g_{mk}\,\hat
u^m_{x^0}\,u^k_{t}\right)\left(\sum\limits_{m=0}^{3}\sum\limits_{k=0}^{3}g_{mk}\,
u^m_{t}\,u^k_{t}\right)^{-1}g\left(\frac{c^2}{2}-\sum\limits_{m=0}^{3}\sum\limits_{k=0}^{3}\frac{c^2}{2}g_{mk}\,
u^m_{t}\,u^k_{t}\right)
\\
+\frac{c^4}{4\pi
G}\left(\sum_{n=0}^{3}\sum_{k=0}^{3}\sum_{m=0}^{3}\sum_{p=0}^{3}\frac{1}{4}g^{mn}g^{pk}\left(\frac{\partial
s_p}{\partial x^m}-\frac{\partial s_m}{\partial
x^p}\right)\left(\frac{\partial s_k}{\partial x^n}-\frac{\partial
s_n}{\partial x^k}\right)\right)\\
-\frac{c^4\beta}{16\pi
G}\left(\left(\sum_{i=0}^{3}\sum_{j=0}^{3}g^{ij}\left(\delta_j
v_i+\delta_i
v_j\right)\right)^2-\sum_{k=0}^{3}\sum_{j=0}^{3}\sum_{m=0}^{3}\sum_{n=0}^{3}\left(\delta_j
v_k+\delta_k v_j\right)g^{km}g^{jn}\left(\delta_m v_n+\delta_n
v_m\right)\right) ,
\end{multline}
where the right hand side is written in the covariant form and the
second equality is valid in every curvilinear coordinate system with
$\frac{\partial t}{\partial x^0}>0$, $(j^0,j^1,j^2,j^3)$ is the
four-vector of the current, $\hat\mu$, $u^j_{t}$ and $u^j_{x^0}$ are
the same as in
\er{btfffjhgjghghijhhjhhhjjkjkkjjhhhjkjjkjhjjkkjklkklijhjhjkmhggmhjhjhkkkint},
and $\vec h$ and $\Phi_0$ are given by
\begin{equation}\label{btfffygtgyggyDDnw22jkhk22mm3488jkjljljujhkhint}
\vec h=\vec v-\vec
k\quad\text{and}\quad\Phi_0=-\frac{1}{2c}\left|\vec v-\vec
k\right|^2.
\end{equation}
Moreover in the particular case of the relativistic-like choice
$g(s):=-c^2\sqrt{1-\frac{2s}{c^2}}$, by
\er{btfffjhgjghghijhhjhhhjjkjkkjjhhhjkjjkjhjjkkjklkklijhjhjkmhggmint}
we can write an alternative to
\er{vhfffngghkjgghDDmmioiuiujhhjhjjhyryrrint} as:
\begin{multline}\label{vhfffngghkjgghDDmmioiuiujhhjhjjhyryrrjkjhhhint}
\frac{1}{8\pi}\left|-\nabla_{\vec
x}\Psi-\frac{1}{c}\frac{\partial\vec A}{\partial t}+\frac{1}{c}\vec
v\times curl_{\vec x}\vec A\right|^2-\frac{1}{8\pi}\left|curl_{\vec
x}\vec A\right|^2-\left(\rho\Psi-\frac{1}{c}\vec A\cdot\vec j\right)
+\mu g\left(\frac{1}{2}\left|\vec u-\vec v\right|^2\right)
\\
-\frac{c^2}{8\pi G}\left|-\nabla_{\vec x}\Phi_0
-\frac{1}{c}\frac{\partial\vec h}{\partial t}
+\frac{1}{c}\vec v\times curl_{\vec x}\vec h-\frac{1}{c}\nabla_{\vec
x}\left(\vec h\cdot\vec v\right)\right|^2+\frac{c^2}{8\pi
G}\left|curl_{\vec x}\vec h\right|^2\\+\frac{c^2\beta}{4\pi
G}\left(\frac{1}{4}\left|d_{\vec x}\vec v+\left\{d_{\vec x}\vec
v\right\}^T\right|^2-\left|\Div_{\vec x}\vec v\right|^2\right)
=L_{ge}=
\\
\frac{1}{4\pi}\left(-\sum_{n=0}^{3}\sum_{k=0}^{3}\sum_{m=0}^{3}\sum_{p=0}^{3}\frac{1}{4}g^{mn}g^{pk}\left(\frac{\partial
A_p}{\partial x^m}-\frac{\partial A_m}{\partial
x^p}\right)\left(\frac{\partial A_k}{\partial x^n}-\frac{\partial
A_n}{\partial x^k}\right)-\sum_{k=0}^{3}4\pi j^k A_k\right)\\
-\frac{c^2\hat \mu}{\sqrt{\left|\text{det}\,G\right|}}
\sqrt{\left(\sum\limits_{m=0}^{3}\sum\limits_{k=0}^{3}g_{mk}\,\hat
u^m_{x^0}\,\hat u^k_{x^0}\right)}
\\
+\frac{c^4}{4\pi
G}\left(\sum_{n=0}^{3}\sum_{k=0}^{3}\sum_{m=0}^{3}\sum_{p=0}^{3}\frac{1}{4}g^{mn}g^{pk}\left(\frac{\partial
s_p}{\partial x^m}-\frac{\partial s_m}{\partial
x^p}\right)\left(\frac{\partial s_k}{\partial x^n}-\frac{\partial
s_n}{\partial x^k}\right)\right)\\
-\frac{c^4\beta}{16\pi
G}\left(\left(\sum_{i=0}^{3}\sum_{j=0}^{3}g^{ij}\left(\delta_j
v_i+\delta_i
v_j\right)\right)^2-\sum_{k=0}^{3}\sum_{j=0}^{3}\sum_{m=0}^{3}\sum_{n=0}^{3}\left(\delta_j
v_k+\delta_k v_j\right)g^{km}g^{jn}\left(\delta_m v_n+\delta_n
v_m\right)\right) ,
\end{multline}
where the right hand side is written in the covariant form and the
second equality is valid in every curvilinear coordinate system with
$\frac{\partial t}{\partial x^0}>0$. On the other hand, in the case
of fully non-relativistic Lagrangian, where
$g(s)=\left(s-\frac{c^2}{2}\right)$, we can write an alternative to
\er{vhfffngghkjgghDDmmioiuiujhhjhjjhyryrrint} as:
\begin{multline}\label{vhfffngghkjgghDDmmioiuiujhhjhjjhyryrrjkjhhhjklkjhhhint}
\frac{1}{8\pi}\left|-\nabla_{\vec
x}\Psi-\frac{1}{c}\frac{\partial\vec A}{\partial t}+\frac{1}{c}\vec
v\times curl_{\vec x}\vec A\right|^2-\frac{1}{8\pi}\left|curl_{\vec
x}\vec A\right|^2-\left(\rho\Psi-\frac{1}{c}\vec A\cdot\vec j\right)
+\mu g\left(\frac{1}{2}\left|\vec u-\vec v\right|^2\right)
\\
-\frac{c^2}{8\pi G}\left|-\nabla_{\vec x}\Phi_0
-\frac{1}{c}\frac{\partial\vec h}{\partial t}
+\frac{1}{c}\vec v\times curl_{\vec x}\vec h-\frac{1}{c}\nabla_{\vec
x}\left(\vec h\cdot\vec v\right)\right|^2+\frac{c^2}{8\pi
G}\left|curl_{\vec x}\vec h\right|^2\\+\frac{c^2\beta}{4\pi
G}\left(\frac{1}{4}\left|d_{\vec x}\vec v+\left\{d_{\vec x}\vec
v\right\}^T\right|^2-\left|\Div_{\vec x}\vec v\right|^2\right)
=L_{ge}=
\\
\frac{1}{4\pi}\left(-\sum_{n=0}^{3}\sum_{k=0}^{3}\sum_{m=0}^{3}\sum_{p=0}^{3}\frac{1}{4}g^{mn}g^{pk}\left(\frac{\partial
A_p}{\partial x^m}-\frac{\partial A_m}{\partial
x^p}\right)\left(\frac{\partial A_k}{\partial x^n}-\frac{\partial
A_n}{\partial x^k}\right)-\sum_{k=0}^{3}4\pi j^k A_k\right)\\
-\frac{c^2\hat \mu}{\sqrt{\left|\text{det}\,G\right|}}
\left(\sum\limits_{m=0}^{3}\sum\limits_{k=0}^{3}g_{mk}\,\hat
u^m_{x^0}\,u^k_{t}\right)
\\
+\frac{c^4}{4\pi
G}\left(\sum_{n=0}^{3}\sum_{k=0}^{3}\sum_{m=0}^{3}\sum_{p=0}^{3}\frac{1}{4}g^{mn}g^{pk}\left(\frac{\partial
s_p}{\partial x^m}-\frac{\partial s_m}{\partial
x^p}\right)\left(\frac{\partial s_k}{\partial x^n}-\frac{\partial
s_n}{\partial x^k}\right)\right)\\
-\frac{c^4\beta}{16\pi
G}\left(\left(\sum_{i=0}^{3}\sum_{j=0}^{3}g^{ij}\left(\delta_j
v_i+\delta_i
v_j\right)\right)^2-\sum_{k=0}^{3}\sum_{j=0}^{3}\sum_{m=0}^{3}\sum_{n=0}^{3}\left(\delta_j
v_k+\delta_k v_j\right)g^{km}g^{jn}\left(\delta_m v_n+\delta_n
v_m\right)\right) ,
\end{multline}
where again the right hand side is written in the covariant form and
the second equality is valid in every curvilinear coordinate system
with $\frac{\partial t}{\partial x^0}>0$.

Note that the four-vector of the gravitational potential
$(v^0,v^1,v^2,v^3)$ and the four-covector of the genuine gravity
$(s_0,s_1,s_2,s_3)$ are \underline{not} independent arguments of
$L_{ge}$. Moreover, in the general non-cartesian or curvilinear
coordinate system, knowing $(v^0,v^1,v^2,v^3)$ and thus also knowing
$\{g_{ij}\}_{0\leq i,j\leq 3}$, as before in
\er{huohuioy89gjjhjffffff3478zzrrZZZhjhhjhhjjhhffGGhjjhuiiuihhjkkoioirrkhrrjggggjrrjiokuukuppkpklhjjhj88khhjhint},
we can find $(s_0,s_1,s_2,s_3)$ by:
\begin{equation}\label{huohuioy89gjjhjffffff3478zzrrZZZhjhhjhhjjhhffGGhjjhuiiuihhjkkoioirrkhrrjggggjrrjiokuukuppkpklhjjhj88khhjhjkojjint}
s_j=\frac{1}{2}\left(\sum_{m=0}^{3} g_{jm}k^m-\sum_{m=0}^{3}
J_{jm}v^m\right)\quad \forall j=0,1,2,3.
\end{equation}
where $(k^0,k^1,k^2,k^3)$ is the four-vector of the potential of
inertia and $\{J_{ij}\}_{0\leq i,j\leq 3}$ is the covariant
kinematic pseudo-metric tensor of inertia, that are prescribed in
every cartesian, non-cartesian or curvilinear coordinate system. On
the other hand, knowing $(s_0,s_1,s_2,s_3)$, we can find
$\{g_{ij}\}_{0\leq i,j\leq 3}$ and $(v^0,v^1,v^2,v^3)$ by the
following covariant identities
\begin{equation}\label{hoyuiouigyfghgjh3478zzrrZZffGGjkkjojj88iihhkhk88jklhkhjkhk8899ouiuuiint}
{g}_{ij}=J_{ij}+\left(s_i\left(c\,\frac{\partial t}{\partial
x^j}\right)+\left(c\,\frac{\partial t}{\partial
x^i}\right)s_j\right)\quad\quad\forall\,i,j=0,1,2,3
\end{equation}
(see
\er{hoyuiouigyfghgjh3478zzrrZZffGGjkkjojj88iihhkhk88jklhkhjkhk8899int}),
and
\begin{equation}\label{hoyuiouigyfghgjh3478zzrrZZffGGjkkjojj88jklkhhk88khjhgh8899kpkjklint}
v^j-k^j=\sum_{m=0}^{3}{\Theta}^{jm}s_m=\sum_{m=0}^{3}\left(k^jk^m-J^{jm}\right)s_m\quad\quad\forall\,j=0,1,2,3
\end{equation}
(see
\er{hoyuiouigyfghgjh3478zzrrZZffGGjkkjojj88jklkhhk88khjhgh8899int}).

Once we wrote the Lagrangian density
$L_{ge}:=L_{ge1}\left(v^k,A^k,x^k\right)_{k=0,\ldots
4}=L_{ge2}\left(s^k,A^k,x^k\right)_{k=0,\ldots 4}$ as a covariant
scalar, under the changes of curvilinear coordinate systems, we can
write a covariant Lagrangian as:
\begin{multline}\label{uuuuuuuujlkjjklkjkint}
J_{ge1}(v^k,A^k)=J_{ge1}(s^k,A^k):=\int_{(x^0,x^1,x^2,x^3)}L_{ge1}\left(v^k,A^k,x^k\right)\,\sqrt{\left|\text{det}\,G\right|}\,dx^0dx^1dx^2dx^3
\\=\int_{(x^0,x^1,x^2,x^3)}L_{ge2}\left(s^k,A^k,x^k\right)\,\sqrt{\left|\text{det}\,G\right|}\,dx^0dx^1dx^2dx^3.
\end{multline}
Although we need a term $\sqrt{\left|\text{det}\,G\right|}$ for the
covariance of the Lagrangian in curvilinear coordinate systems, in
the cartesian coordinate systems we always have
$\sqrt{\left|\text{det}\,G\right|}=1$.

Finally, note again, that the four-vector of the gravitational
potential, as an independent argument, is restricted by
\er{fgjfjhfgfgfgfhihhjhhjjhhjhjjhjhhjhjkjhint}:
\begin{equation}\label{fgjfjhfgfgfgfhihhjhhjjhhjhjjhjhhjhjkjhojjjjjlljjint}
\sum_{k=0}^{3}c\frac{\partial t}{\partial x^k}v^k
=1.
\end{equation}
So the four-vector $(v^0,v^1,v^2,v^3)$ actually contains three
independent scalar functions similarly as in cartesian coordinate
systems where we have $v^0=1$. On the other hand if we consider
$(s_0,s_1,s_2,s_3)$ instead of $(v^0,v^1,v^2,v^3)$ , as an
independent argument, then by
\er{fgjfjhgghhgjghjhjkkkkgjghghuiiiuujhjhljhgugii88int} and
\er{hoyuiouigyfghgjh3478zzrrZZffGGjkkjojj88iihhkhk88jklhkhjkhk8899ouiuuiint}
it is restricted by the following identity
\begin{equation}\label{fgjfjhgghhgjghjhjkkkkgjghghuiiiuujhjhljhgugii88jklkljjkk;kklint}
\text{det}\,\left(\left\{J_{ij}+s_i\left(c\,\frac{\partial
t}{\partial x^j}\right)+\left(c\,\frac{\partial t}{\partial
x^i}\right)s_j\right\}_{0\leq i,j\leq
3}\right)=\text{det}\,\left(\{J_{ij}\}_{0\leq i,j\leq 3}\right),
\end{equation}
which is valid in every cartesian, non-cartesian or curvilinear
coordinate system. So the four-covector $(s_0,s_1,s_2,s_3)$ also
contains three independent scalar functions.

\subsection{Relativistic-like Dirac equation} As in
\er{btfffygtgyggyijhhkkSYSPNuhuygyygyggyuyyuyuykuhghgjjojiyyyuyuuyyint}
consider the relativistic-like Lagrangian of the motion of the
particle with mass $m$ and charge $\sigma$ in the outer
gravitational and electromagnetic fields and additional field with
potential $V(\vec x,t)$:
\begin{multline}\label{btfffygtgyggyijhhkkSYSPNuhuygyygyggyuyyuyuykuhghgjjojiyyyuyuuyyKKint}
J_{rl}(\vec r)=\int_0^T L_0\left(\frac{d\vec r}{dt},\vec
r,t\right)dt:=\\ \int_0^T\left\{
-mc^2\sqrt{1-\frac{1}{c^2}\left|\frac{d\vec r}{dt}-\vec v(\vec
r,t)\right|^2}-\sigma\left(\Psi(\vec r,t)-\frac{1}{c}\vec A(\vec
r,t)\cdot\frac{d\vec r}{dt}\right)+V\left(\vec r,t\right)\right\}dt.
\end{multline}
Next define the generalized momentum of the particle by
\begin{equation}\label{guytyurtydftyiujhSYShmhmKKint}
\vec P:=\nabla_{\vec r'}L_0\left(\vec r',\vec
r,t\right)=m\left(1-\frac{1}{c^2}\left|\frac{d\vec r}{dt}-\vec
v(\vec r,t)\right|^2\right)^{-\frac{1}{2}}\left(\frac{d\vec
r}{dt}-\vec v(\vec r,t)\right)+\frac{\sigma}{c}\vec A(\vec r,t).
\end{equation}
Then
\begin{equation}\label{guytyurtydftyiujhioyhuyhSYShmhmuiiuiuKKint}
\frac{d\vec r}{dt}=\left(1+\frac{1}{c^2}\left|\frac{1}{m}\vec
P-\frac{\sigma}{mc}\vec A(\vec
r,t)\right|^2\right)^{-\frac{1}{2}}\left(\frac{1}{m}\vec
P-\frac{\sigma}{mc}\vec A(\vec r,t)\right)+\vec v(\vec r,t).
\end{equation}
Thus, if we consider a Hamiltonian
\begin{equation}\label{vhfffngghkjgghfjjhyjjfgSYShmhmKKint}
H_0\left(\vec P,\vec r,t\right):=\vec P\cdot\frac{d\vec
r}{dt}-L_0\left(\frac{d\vec r}{dt},\vec r,t\right),
\end{equation}
then we deduce, that the relativistic-like Hamiltonian for a
macro-particles has the form:
\begin{multline}\label{vhfffngghkjgghfjjghghghSYShmyuuiiuuhmhmiopoopKKint}
H_0\left(\vec P,\vec r,t\right)=
mc^2\left(1+\frac{1}{c^2}\left|\frac{1}{m}\vec
P-\frac{\sigma}{mc}\vec A(\vec r,t)\right|^2\right)^{\frac{1}{2}}
+\sigma\left(\Psi(\vec r,t)-\frac{1}{c}\vec v(\vec r,t)\cdot\vec
A(\vec r,t)\right)\\-V\left(\vec r,t\right)+\vec v(\vec
r,t)\cdot\vec P
\end{multline}
(see section \ref{fgffggfggfhhjh} for the details). In particular,
if similarly to
\er{fgjfjhgghhgjghjhjijhojihjhjjijhjjjljljpkhkhjgjgint} we define
the four-dimensional generalized momentum $(P_0,P_1,P_2,P_3)$ as:
\begin{equation}\label{fgjfjhgghhgjghjhjijhojihjhjjijhjjjljljpkhkhjgjgjjint}
(P_0,P_1,P_2,P_3):=\left(\frac{1}{c}H_0,-\vec
P\right)\quad\text{where}\quad
P_0=\frac{1}{c}H_0\;\;\text{and}\;\;(P_1,P_2,P_3)=-\vec P,
\end{equation}
Then, since by \er{guytyurtydftyiujhSYShmhmKKint} and
\er{vhfffngghkjgghfjjghghghSYShmyuuiiuuhmhmiopoopKKint}, under the
change of non-inertial cartesian coordinate system $H_0$ and $\vec
P$ transform as
\begin{equation}\label{vhfffngghhjghhgPPNghghghutghfflklhjkjhjhjjgjkghhjhhjhjkhjghlkklljjint}
\begin{cases}
H_0'=
H_0+\left(\frac{dA}{dt}(t)\cdot\vec x+\frac{d\vec
z}{dt}(t)\right)\cdot\left(A(t)\cdot\vec P\right)
\\
\vec P'=A(t)\cdot \vec P,
\end{cases}
\end{equation}
by comparing
\er{vhfffngghhjghhgPPNghghghutghfflklhjkjhjhjjgjkghhjhhjhjkhjghlkklljjint}
with \er{fgjfjhgghhgjghjhjijhojpiiihjhjuioujint} we deduce that the
four-dimensional momentum $(P_0,P_1,P_2,P_3)$ is a
four-\underline{covector} on the group $\mathcal{S}_0$ that is the
group of changes of cartesian non-inertial coordinate systems.

See subsection \ref{hgggggjjhhk} for the Relativistic-like
Lagrangian and Hamiltonian of the motion of the system of $n$
particles. See also subsection \ref{gghfhfhfhfh11} for the
Liouville's equation for a system of $n$ relativistic-like classical
particles.

Next consider the motion of a spin-half quantum relativistic-like
micro-particle with inertial mass $m$ and the charge $\sigma$ in the
outer gravitational and electromagnetical field with characteristics
$\vec v(\vec x,t)$, $\vec A(\vec x,t)$ and $\Psi(\vec x,t)$ and
additional conservative field with potential $V(\vec y,t)$. The
evolution equation for this particle is
\begin{equation}\label{vhfffngghkjgghfjjghghghhjghjgghkghggkghghjghSYSPNnnkkJJint}
i\hbar\frac{\partial\psi}{\partial t}=\hat H_0\cdot\psi,
\end{equation}
where $\psi(\vec x,t)=\left(\psi_1(\vec x,t),\psi_2(\vec
x,t)\right)\in\mathbb{C}^2\times \mathbb{C}^2$ is a four-component
wave function and $\hat H_0$ is the Hamiltonian operator. Since the
relativistic-like Hamiltonian for a macro-particles has the form
\er{vhfffngghkjgghfjjghghghSYShmyuuiiuuhmhmiopoopKKint},
%
%
%
analogously to the usual Dirac Hamiltonian operator, we built the
Hermitian Hamiltonian operator as $\hat H_0\cdot\psi=\left(\hat
H_1\cdot\psi,\hat H_2\cdot\psi\right)$, where
\begin{multline}\label{vhfffngghkjgghfjjghghghSYShmyuuiiuuhmhmiopoopnniukjhjkkkJJint}
\hat H_1\cdot\psi= mc^2\psi_1-c\,\vec S\cdot\left(i\hbar\nabla_{\vec
x}\psi_2+\frac{\sigma}{c}\vec A(\vec x,t)\psi_2\right)
%
%
%
+\sigma\left(\Psi(\vec x,t)-\frac{1}{c}\vec v(\vec x,t)\cdot\vec
A(\vec x_j,t)\right)\psi_1\\-V\left(\vec
x,t\right)\psi_1-\frac{i\hbar}{2}div_{\vec x}\left\{\psi_1\vec
v(\vec x,t)\right\}-\frac{i\hbar}{2}\nabla_{\vec x}\psi_1\cdot\vec
v(\vec x,t)+\frac{\hbar}{4}\vec S\cdot\left(curl_{\vec x}\vec v(\vec
x,t)\,\psi_1\right)\\-\frac{(g-1)\sigma\hbar}{2mc}\vec
S\cdot\left(curl_{\vec x}\vec A(\vec x,t)\,\psi_1\right),
\end{multline}
and
\begin{multline}\label{vhfffngghkjgghfjjghghghSYShmyuuiiuuhmhmiopoopnniukjhjkhhjkkJJint}
\hat H_2\cdot\psi= -mc^2\psi_2-c\,\vec
S\cdot\left(i\hbar\nabla_{\vec x}\psi_1+\frac{\sigma}{c}\vec A(\vec
x,t)\psi_1\right)
%
%
%
+\sigma\left(\Psi(\vec x,t)-\frac{1}{c}\vec v(\vec x,t)\cdot\vec
A(\vec x,t)\right)\psi_2\\-V\left(\vec
x,t\right)\psi_2-\frac{i\hbar}{2}div_{\vec x}\left\{\psi_2\vec
v(\vec x,t)\right\}-\frac{i\hbar}{2}\nabla_{\vec x}\psi_2\cdot\vec
v(\vec x,t)+\frac{\hbar}{4}\vec S\cdot\left(curl_{\vec x}\vec v(\vec
x,t)\,\psi_2\right)\\-\frac{(g-1)\sigma\hbar}{2mc}\vec
S\cdot\left(curl_{\vec x}\vec A(\vec x,t)\,\psi_2\right),
\end{multline}
where $\vec S:=(S_1,S_2,S_3)$ and
$$S_1=\left(\begin{matrix}0&1\\1&0\end{matrix}\right),\quad
S_2=\left(\begin{matrix}0&-i\\i&0\end{matrix}\right)\quad
S_3=\left(\begin{matrix}1&0\\0&-1\end{matrix}\right)$$ are Pauli
matrices and $g$ is a constant that depends on the type of the
particle (for electron we have $g=1$). As before for the
Schr\"{o}dinger-Pauli equation, we added an additional term to the
Hamiltonian, namely $\frac{\hbar}{4}\vec S\cdot\left(curl_{\vec
x}\vec v(\vec x,t)\,\psi\right)$. Although this term vanishes in
inertial coordinate systems, it provides however invariance of our
Dirac-type equation, under the change of non-inertial coordinate
systems as we will see below. Thus, we have the following two
evolution equations that we call together Dirac system of equations:
\begin{multline}\label{vhfffngghkjgghfjjghghghSYShmyuuiiuuhmhmiopoopnniukjhjkk;l;lkkkJJint}
i\hbar\frac{\partial\psi_1}{\partial t}=mc^2\psi_1-c\,\vec
S\cdot\left(i\hbar\nabla_{\vec x}\psi_2+\frac{\sigma}{c}\vec A(\vec
x,t)\psi_2\right)
%
%
%
+\sigma\left(\Psi(\vec x,t)-\frac{1}{c}\vec v(\vec x,t)\cdot\vec
A(\vec x_j,t)\right)\psi_1\\-V\left(\vec
x,t\right)\psi_1-\frac{i\hbar}{2}div_{\vec x}\left\{\psi_1\vec
v(\vec x,t)\right\}-\frac{i\hbar}{2}\nabla_{\vec x}\psi_1\cdot\vec
v(\vec x,t)+\frac{\hbar}{4}\vec S\cdot\left(curl_{\vec x}\vec v(\vec
x,t)\,\psi_1\right)\\-\frac{(g-1)\sigma\hbar}{2mc}\vec
S\cdot\left(curl_{\vec x}\vec A(\vec x,t)\,\psi_1\right),
\end{multline}
and
\begin{multline}\label{vhfffngghkjgghfjjghghghSYShmyuuiiuuhmhmiopoopnniukjhjkhhjl;kkJJint}
i\hbar\frac{\partial\psi_2}{\partial t}=-mc^2\psi_2-c\,\vec
S\cdot\left(i\hbar\nabla_{\vec x}\psi_1+\frac{\sigma}{c}\vec A(\vec
x,t)\psi_1\right)
%
%
%
+\sigma\left(\Psi(\vec x,t)-\frac{1}{c}\vec v(\vec x,t)\cdot\vec
A(\vec x,t)\right)\psi_2\\-V\left(\vec
x,t\right)\psi_2-\frac{i\hbar}{2}div_{\vec x}\left\{\psi_2\vec
v(\vec x,t)\right\}-\frac{i\hbar}{2}\nabla_{\vec x}\psi_2\cdot\vec
v(\vec x,t)+\frac{\hbar}{4}\vec S\cdot\left(curl_{\vec x}\vec v(\vec
x,t)\,\psi_2\right)\\-\frac{(g-1)\sigma\hbar}{2mc}\vec
S\cdot\left(curl_{\vec x}\vec A(\vec x,t)\,\psi_2\right).
\end{multline}
Then we can rewrite Dirac equations as:
\begin{multline}\label{vhfffngghkjgghfjjghghghSYShmyuuiiuuhmhmiopoopnniukjhjkk;l;lkkkJJjjkjkint}
i\hbar\left(\frac{\partial\psi_1}{\partial t}+\frac{1}{2}div_{\vec
x}\left\{\psi_1\vec v(\vec x,t)\right\}+\frac{1}{2}\nabla_{\vec
x}\psi_1\cdot\vec v(\vec x,t)\right)=mc^2\psi_1-c\,\vec
S\cdot\left(i\hbar\nabla_{\vec x}\psi_2+\frac{\sigma}{c}\vec A(\vec
x,t)\psi_2\right)
%
%
%
\\+\sigma\left(\Psi(\vec x,t)-\frac{1}{c}\vec v(\vec x,t)\cdot\vec
A(\vec x,t)\right)\psi_1-V\left(\vec
x,t\right)\psi_1+\frac{\hbar}{4}\vec S\cdot\left(curl_{\vec x}\vec
v(\vec x,t)\,\psi_1\right)\\-\frac{(g-1)\sigma\hbar}{2mc}\vec
S\cdot\left(curl_{\vec x}\vec A(\vec x,t)\,\psi_1\right),
\end{multline}
and
\begin{multline}\label{vhfffngghkjgghfjjghghghSYShmyuuiiuuhmhmiopoopnniukjhjkhhjl;kkJJkllkklint}
i\hbar\left(\frac{\partial\psi_2}{\partial t}+\frac{1}{2}div_{\vec
x}\left\{\psi_2\vec v(\vec x,t)\right\}+\frac{1}{2}\nabla_{\vec
x}\psi_2\cdot\vec v(\vec x,t)\right)=-mc^2\psi_2-c\,\vec
S\cdot\left(i\hbar\nabla_{\vec x}\psi_1+\frac{\sigma}{c}\vec A(\vec
x,t)\psi_1\right)
%
%
%
\\+\sigma\left(\Psi(\vec x,t)-\frac{1}{c}\vec v(\vec x,t)\cdot\vec
A(\vec x,t)\right)\psi_2-V\left(\vec
x,t\right)\psi_2+\frac{\hbar}{4}\vec S\cdot\left(curl_{\vec x}\vec
v(\vec x,t)\,\psi_2\right)\\-\frac{(g-1)\sigma\hbar}{2mc}\vec
S\cdot\left(curl_{\vec x}\vec A(\vec x,t)\,\psi_2\right).
\end{multline}
Then similarly to the proof of Theorem \ref{gjghghgghgintintrrZZint}
about the invariance of Shr\"{o}dinger-Pauli equation we can prove
the following Theorem for Dirac equations:
\begin{theorem}\label{gjghghgghgintintrrZZjkhint}
Consider that the change of some cartesian coordinate system $(*)$
to another cartesian coordinate system $(**)$ is given by
\begin{equation}\label{noninchgravortbstrjgghguittu2intrrrZZjkjjkjkjkint}
\begin{cases}
\vec x'=A(t)\cdot\vec x+\vec z(t),\\
t'=t,
\end{cases}
\end{equation}
where $A(t)\in SO(3)$ is a rotation. Next, assume that in the
coordinate system $(**)$ we observe a validity of the Dirac
equations of the form:
\begin{multline}\label{vhfffngghkjgghfjjghghghSYShmyuuiiuuhmhmiopoopnniukjhjkk;l;lkkkJJjjkjk;lk;klint}
i\hbar\left(\frac{\partial\psi'_1}{\partial t'}+\frac{1}{2}div_{\vec
x'}\left\{\psi'_1\vec v'(\vec x',t')\right\}+\frac{1}{2}\nabla_{\vec
x'}\psi'_1\cdot\vec v'(\vec x',t')\right)=m'c^2\psi'_1-c\,\vec
S\cdot\left(i\hbar\nabla_{\vec x'}\psi'_2+\frac{\sigma'}{c}\vec
A'(\vec x',t)\psi'_2\right)
%
%
%
\\+\sigma'\left(\Psi'(\vec x',t')-\frac{1}{c}\vec v'(\vec x',t')\cdot\vec
A'(\vec x',t')\right)\psi'_1-V'\left(\vec
x',t'\right)\psi'_1+\frac{\hbar}{4}\vec S\cdot\left(curl_{\vec
x'}\vec v'(\vec
x',t')\,\psi'_1\right)\\-\frac{(g'-1)\sigma'\hbar}{2m'c}\vec
S\cdot\left(curl_{\vec x'}\vec A'(\vec x',t')\,\psi'_1\right),
\end{multline}
and
\begin{multline}\label{vhfffngghkjgghfjjghghghSYShmyuuiiuuhmhmiopoopnniukjhjkhhjl;kkJJkllkklkkint}
i\hbar\left(\frac{\partial\psi'_2}{\partial t'}+\frac{1}{2}div_{\vec
x'}\left\{\psi'_2\vec v'(\vec x',t')\right\}+\frac{1}{2}\nabla_{\vec
x'}\psi'_2\cdot\vec v'(\vec x',t')\right)=\\-m'c^2\psi'_2-c\,\vec
S\cdot\left(i\hbar\nabla_{\vec x'}\psi'_1+\frac{\sigma'}{c}\vec
A'(\vec x',t')\psi'_1\right)
%
%
%
+\sigma'\left(\Psi'(\vec x',t')-\frac{1}{c}\vec v'(\vec
x',t')\cdot\vec A'(\vec x',t')\right)\psi'_2\\-V'\left(\vec
x',t'\right)\psi'_2+\frac{\hbar}{4}\vec S\cdot\left(curl_{\vec
x'}\vec v'(\vec
x',t')\,\psi'_2\right)\\-\frac{(g'-1)\sigma'\hbar}{2m'c}\vec
S\cdot\left(curl_{\vec x'}\vec A'(\vec x',t')\,\psi'_2\right),
\end{multline}
where $\psi=(\psi_1,\psi_2)\in\mathbb{C}^2\times \mathbb{C}^2$ is a
four-component wave function. Then in the coordinate system $(*)$ we
have the validity of Dirac equations of the same as
\er{vhfffngghkjgghfjjghghghSYShmyuuiiuuhmhmiopoopnniukjhjkk;l;lkkkJJjjkjk;lk;klint}
and
\er{vhfffngghkjgghfjjghghghSYShmyuuiiuuhmhmiopoopnniukjhjkhhjl;kkJJkllkklkkint}
form:
\begin{multline}\label{vhfffngghkjgghfjjghghghSYShmyuuiiuuhmhmiopoopnniukjhjkk;l;lkkkJJjjkjkjkjkint}
i\hbar\left(\frac{\partial\psi_1}{\partial t}+\frac{1}{2}div_{\vec
x}\left\{\psi_1\vec v(\vec x,t)\right\}+\frac{1}{2}\nabla_{\vec
x}\psi_1\cdot\vec v(\vec x,t)\right)=mc^2\psi_1-c\,\vec
S\cdot\left(i\hbar\nabla_{\vec x}\psi_2+\frac{\sigma}{c}\vec A(\vec
x,t)\psi_2\right)
%
%
%
\\+\sigma\left(\Psi(\vec x,t)-\frac{1}{c}\vec v(\vec x,t)\cdot\vec
A(\vec x,t)\right)\psi_1-V\left(\vec
x,t\right)\psi_1+\frac{\hbar}{4}\vec S\cdot\left(curl_{\vec x}\vec
v(\vec x,t)\,\psi_1\right)\\-\frac{(g-1)\sigma\hbar}{2mc}\vec
S\cdot\left(curl_{\vec x}\vec A(\vec x,t)\,\psi_1\right),
\end{multline}
and
\begin{multline}\label{vhfffngghkjgghfjjghghghSYShmyuuiiuuhmhmiopoopnniukjhjkhhjl;kkJJkllkkljjjint}
i\hbar\left(\frac{\partial\psi_2}{\partial t}+\frac{1}{2}div_{\vec
x}\left\{\psi_2\vec v(\vec x,t)\right\}+\frac{1}{2}\nabla_{\vec
x}\psi_2\cdot\vec v(\vec x,t)\right)=-mc^2\psi_2-c\,\vec
S\cdot\left(i\hbar\nabla_{\vec x}\psi_1+\frac{\sigma}{c}\vec A(\vec
x,t)\psi_1\right)
%
%
%
\\+\sigma\left(\Psi(\vec x,t)-\frac{1}{c}\vec v(\vec x,t)\cdot\vec
A(\vec x,t)\right)\psi_2-V\left(\vec
x,t\right)\psi_2+\frac{\hbar}{4}\vec S\cdot\left(curl_{\vec x}\vec
v(\vec x,t)\,\psi_2\right)\\-\frac{(g-1)\sigma\hbar}{2mc}\vec
S\cdot\left(curl_{\vec x}\vec A(\vec x,t)\,\psi_2\right),
\end{multline}
provided that
\begin{equation}\label{yuythfgfyftydtydtydtyddyyyhhddhhhredPPN111hgghjgintintrrZZkjhhint}
\begin{cases}
V'=V,\\
\sigma'=\sigma,\\
g'=g,
\\
m'=m,\\
\vec v'=A(t)\cdot \vec v+\frac{dA}{dt}(t)\cdot\vec x+\frac{d\vec z}{dt}(t),\\
\vec A'=A(t)\cdot \vec A,\\
\Psi'-\vec v'\cdot\vec A'=\Psi-\vec v\cdot\vec A,\\
\psi'_1=U(t)\cdot\psi_1,\\
\psi'_2=U(t)\cdot\psi_2,
\end{cases}
\end{equation}
where, as before, $U(t)\in SU(2)$ is characterized by:
\begin{equation}\label{gyfyfgfgfgghZZjkjkint}
U^*(t)\cdot\vec S\cdot U(t)=A(t)\cdot\vec S,
\end{equation}
that means
\begin{multline*}
\left(U^*(t)\cdot S_1\cdot U(t),U^*(t)\cdot S_2\cdot
U(t),U^*(t)\cdot S_3\cdot U(t)\right)=\\
\left(a_{11}(t)S_1+a_{12}(t)S_2+a_{13}(t)S_3\,,\,a_{21}(t)S_1+a_{22}(t)S_2+a_{23}(t)S_3\,,\,a_{31}(t)S_1+a_{32}(t)S_2+a_{33}(t)S_3\right),
\end{multline*}
where $A(t)=\left\{a_{mk}(t)\right\}_{\{1\leq m,k\leq 3\}}$.
\end{theorem}

Next, in the case that our particle has a positive energy, define
$$(\phi_1,\phi_2)=\left(e^{-\frac{ic^2mt}{\hbar}}\psi_1,e^{-\frac{ic^2mt}{\hbar}}\psi_2\right).$$
Then, as we show in section \ref{fgffggfggfhhjh}, we rewrite
\er{vhfffngghkjgghfjjghghghSYShmyuuiiuuhmhmiopoopnniukjhjkk;l;lkkkJJint}
and
\er{vhfffngghkjgghfjjghghghSYShmyuuiiuuhmhmiopoopnniukjhjkhhjl;kkJJint}
in the non-relativistic limit as:
\begin{equation}\label{vhfffngghkjgghfjjghghghSYShmyuuiiuuhmhmiopoopnniukjhjkhhjl;jkljhjjkllkjhykkJJint}
\phi_2\approx-\frac{1}{2c m}\vec S\cdot\left(i\hbar\nabla_{\vec
x}\phi_1+\frac{\sigma}{c}\vec A(\vec x,t)\phi_1\right).
%
%
%
\end{equation}
and:
\begin{multline}\label{vhfffngghkjgghfjjghghghSYShmyuuiiuuhmhmiopoopnniukjhjkk;l;lkhjjkihjjhkkJJllkkkint}
i\hbar\frac{\partial\phi_1}{\partial
t}\approx-\frac{\hbar^2}{2m}\Delta_{\vec
x}\phi_1+\frac{i\hbar\sigma}{2mc}div_{\vec x}\left\{\phi_1\vec
A(\vec x,t)\right\}+\frac{i\hbar\sigma}{2mc}\nabla_{\vec
x}\phi_1\cdot\vec A(\vec x,t)+\frac{\sigma^2}{2mc^2}\left|\vec
A(\vec x,t)\right|^2\phi_1\\-\frac{g\sigma\hbar}{2mc}\vec
S\cdot\left(curl_{\vec x}\vec A(\vec x,t)\,\phi_1\right)
%
%
%
+\sigma\left(\Psi(\vec x,t)-\frac{1}{c}\vec v(\vec x,t)\cdot\vec
A(\vec x,t)\right)\phi_1-V\left(\vec
x,t\right)\phi_1\\-\frac{i\hbar}{2}div_{\vec x}\left\{\phi_1\vec
v(\vec x,t)\right\}-\frac{i\hbar}{2}\nabla_{\vec x}\phi_1\cdot\vec
v(\vec x,t)+\frac{\hbar}{4}\vec S\cdot\left(curl_{\vec x}\vec v(\vec
x,t)\,\phi_1\right).
\end{multline}
The last equation coincides with the non-relativistic
Shr\"{o}dinger-Pauli equation, that we studied above.

Next, consider a Lagrangian density $L$ associated with the motion
of a spin-half quantum relativistic-like micro-particle with
inertial mass $m$ and the charge $\sigma$ in the outer gravitational
and electromagnetical field with characteristics $\vec v(\vec x,t)$,
$\vec A(\vec x,t)$ and $\Psi(\vec x,t)$ and additional conservative
field with potential $V(\vec y,t)$:
\begin{multline}\label{vhfffngghkjgghDDmmkkkZZkkkint}
L\left(\psi,\vec x,t\right):=
\frac{i\hbar}{2}\left(\left(\frac{\partial\psi_1}{\partial t}+\vec
v\cdot\nabla_{\vec
x}\psi_1\right)\cdot\bar\psi_1-\psi_1\cdot\left(\frac{\partial\bar\psi_1}{\partial
t}+\vec v\cdot\nabla_{\vec
x}\bar\psi_1\right)\right)\\+\frac{c}{2}\left(\left(\vec
S\cdot\left(i\hbar\nabla_{\vec x}\psi_2+\frac{\sigma}{c}\vec A(\vec
x,t)\psi_2\right)\right)\cdot\bar\psi_1-\psi_1\cdot\left(\vec
S\cdot\left(i\hbar\nabla_{\vec x}\bar\psi_2-\frac{\sigma}{c}\vec
A(\vec x,t)\bar\psi_2\right)\right)\right)\\+
\frac{i\hbar}{2}\left(\left(\frac{\partial\psi_2}{\partial t}+\vec
v\cdot\nabla_{\vec
x}\psi_2\right)\cdot\bar\psi_2-\psi_2\cdot\left(\frac{\partial\bar\psi_2}{\partial
t}+\vec v\cdot\nabla_{\vec
x}\bar\psi_2\right)\right)\\
+\frac{c}{2}\left(\left(\vec S\cdot\left(i\hbar\nabla_{\vec
x}\psi_1+\frac{\sigma}{c}\vec A(\vec
x,t)\psi_1\right)\right)\cdot\bar\psi_2-\psi_2\cdot\left(\vec
S\cdot\left(i\hbar\nabla_{\vec x}\bar\psi_1-\frac{\sigma}{c}\vec
A(\vec x,t)\bar\psi_1\right)\right) \right)
\\-mc^2\left(\psi_1\cdot\bar\psi_1-\psi_2\cdot\bar\psi_2\right)
-\sigma\left(\Psi-\frac{1}{c}\vec v\cdot\vec
A\right)\left(\psi_1\cdot\bar\psi_1+\psi_2\cdot\bar\psi_2\right)+V\left(\vec
x,t\right)\left(\psi_1\cdot\bar\psi_1+\psi_2\cdot\bar\psi_2\right)\\-\frac{\hbar}{4}\left(\vec
S\cdot\left(curl_{\vec x}\vec
v\right)\psi_1\right)\cdot\bar\psi_1-\frac{\hbar}{4}\left(\vec
S\cdot\left(curl_{\vec x}\vec
v\right)\psi_2\right)\cdot\bar\psi_2\\+\frac{(g-1)\sigma\hbar}{2mc}\left(\vec
S\cdot\left(curl_{\vec x}\vec
A\right)\psi_1\right)\cdot\bar\psi_1+\frac{(g-1)\sigma\hbar}{2mc}\left(\vec
S\cdot\left(curl_{\vec x}\vec A\right)\psi_2\right)\cdot\bar\psi_2,
\end{multline}
where $\psi=(\psi_1,\psi_2)\in\mathbb{C}^2\times \mathbb{C}^2$ is a
four-component wave function. Then similarly to the proof of Theorem
\ref{gjghghgghgintintrrZZjkhint} we can prove that $L$ is invariant
under the change of inertial or non-inertial cartesian coordinate
system, given by
\er{noninchgravortbstrjgghguittu2intrrrZZjkjjkjkjkint}, provided
that we take into account
\er{yuythfgfyftydtydtydtyddyyyhhddhhhredPPN111hgghjgintintrrZZkjhhint}.
Moreover, if we consider a functional
\begin{equation}\label{btfffygtgyggyDDmmkkkZZkkkint}
J=\int_0^T\int_{\mathbb{R}^3}L\left(\psi,\vec x,t\right)d\vec x dt,
\end{equation}
then, by \er{vhfffngghkjgghDDmmkkkZZkkkint} we deduce the
Euler-Lagranges equation for \er{btfffygtgyggyDDmmkkkZZkkkint}
coincide with Dirac equations in the form of
\er{vhfffngghkjgghfjjghghghSYShmyuuiiuuhmhmiopoopnniukjhjkk;l;lkkkJJjjkjkint}
and
\er{vhfffngghkjgghfjjghghghSYShmyuuiiuuhmhmiopoopnniukjhjkhhjl;kkJJkllkklint}.
Next, as before, we would like to note that, as before, the
Lagrangian density $L$, defined by
\er{vhfffngghkjgghDDmmkkkZZkkkint} obeys $U(1)$ local symmetry, i.e.
for every scalar field $w:=w(\vec x,t)$ one can easily deduce that
$L$ in \er{vhfffngghkjgghDDmmkkkZZkkkint} is invariant under the
transformation:
\begin{equation}\label{MaxVacFull1bjkgjhjhgjgjgkjfhjfdghghligioiuittrPPNhjkjhkjgghhjjhjintiiihhintspggjjf}
\begin{cases}
\psi_1\,\to\,e^{-\frac{i\sigma w}{c\hbar}}\,\psi_1\\
\psi_2\,\to\,e^{-\frac{i\sigma w}{c\hbar}}\,\psi_2
\\
\Psi\,\to\,\Psi+\frac{1}{c}\frac{\partial w}{\partial t}\\
\vec A\,\to\,\vec A-\nabla_{\vec x}w\\
\vec v\,\to\,\vec v.
\end{cases}
\end{equation}

See also subsection \ref{ghghggjghfgjknew} for generalizations of
Dirac equation for a system of $n$ spin-half stable particles. See
also subsection \ref{hilghjhghfdgdg1133} for the Quantum Liouville's
equation for a finite system of spin-half relativistic-like stable
particles.

\subsection{Spinless relativistic-like particles: the Klein--Gordon
equation} Consider the motion of a relativistic-like spinless
quantum micro-particle with inertial masses $m$ and the charge
$\sigma$ in the outer gravitational and electromagnetical field with
characteristics $\vec v(\vec x,t)$, $\vec A(\vec x,t)$ and
$\Psi(\vec x,t)$ and additional conservative field with potential
$V(\vec x,t)$. The Klein--Gordon equation for this system of
particles is the following:
\begin{multline}\label{vhfffngghkjgghfjjghghghhjghjgghkghgghjhggjjkgfgdiyfgfjkjgjgSYSPNxxint}
\frac{\hbar^2}{2c^2m}\frac{\partial}{\partial
t}\left(\frac{\partial\psi}{\partial t}+\vec v(\vec
x,t)\cdot\nabla_{\vec x}\psi+\frac{i\sigma}{\hbar}\Psi_0(\vec
x,t)\psi\right)\\+\frac{\hbar^2}{2c^2m}\Div_{\vec x}\left\{\vec
v(\vec x,t)\left(\frac{\partial\psi}{\partial t}+\vec v(\vec
x,t)\cdot\nabla_{\vec x}\psi+\frac{i\sigma}{\hbar}\Psi_0(\vec
x,t)\psi\right)\right\}\\
+\frac{\hbar^2}{2c^2m}\left(\frac{i\sigma}{\hbar}\Psi_0(\vec
x,t)\right)\left(\frac{\partial\psi}{\partial t}+\vec v(\vec
x,t)\cdot\nabla_{\vec x}\psi+\frac{i\sigma}{\hbar}\Psi_0(\vec
x,t)\psi\right) +\frac{c^2m}{2}\psi-\frac{\hbar^2}{2m}\Delta_{\vec
x}\psi-V\left(\vec x,t\right)\psi\\+\frac{
i\hbar\sigma}{2mc}div_{\vec x}\left\{\psi\vec A(\vec
x,t)\right\}+\frac{ i\hbar\sigma}{2mc}\vec A(\vec
x,t)\cdot\nabla_{\vec x}\psi+\left(
\frac{\sigma^2}{2mc^2}\left|\vec A(\vec x,t)\right|^2\right)\psi=0,
\end{multline}
where, as before,
\begin{equation}\label{noninchgravortbstrghgggSYSPNxxjjkkjxxint}
\Psi_0(\vec x,t):=\Psi(\vec x,t)-\frac{1}{c}\vec v(\vec
x,t)\cdot\vec A(\vec x,t)
\end{equation}
is the proper electrical potential and $\psi=\psi\left(\vec
x,t\right)\in\mathbb{C}$ is the scalar wave function. Next consider
a change of some non-inertial cartesian coordinate system $(*)$ to
another cartesian coordinate system $(**)$ of the form:
\begin{equation}\label{noninchgravortbstrghgggSYSPNxxint}
\begin{cases}
\vec x'=A(t)\cdot\vec x+\vec z(t),\\
t'=t,
\end{cases}
\end{equation}
where $A(t)\in SO(3)$ is a rotation.
Then, we deduce that  the  Klein--Gordon  equation of the form
\er{vhfffngghkjgghfjjghghghhjghjgghkghgghjhggjjkgfgdiyfgfjkjgjgSYSPNxxint}
is invariant under the change of non-inertial cartesian coordinate
system, provided that under \er{noninchgravortbstrghgggSYSPNxxint}
we have
\begin{equation}\label{vyfgjhgjhvhgghSYSPNxxint}
\begin{cases}
\psi'=\psi
\\
V'=V
\\
\vec A'=A(t)\cdot\vec A
\\
\Psi'_0:=\Psi'-\frac{1}{c}\vec A'\cdot\vec v'=\Psi-\frac{1}{c}\vec
A\cdot\vec v:=\Psi_0.
\end{cases}
\end{equation}

Next defining
\begin{equation}\label{noninchgravortbstrghgggSYSPNxxjjjkjkxxint}
\psi_1\left(\vec x,t\right):=e^{\frac{ic^2mt}{\hbar}}\psi\left(\vec
x,t\right)
\end{equation}
we rewrite
\er{vhfffngghkjgghfjjghghghhjghjgghkghgghjhggjjkgfgdiyfgfjkjgjgSYSPNxxint}
as:
\begin{multline}\label{vhfffngghkjgghfjjghghghhjghjgghkghgghjhggjjkgfgdiyfgfjkjgjgSYSPNxxjkhhkxxhjjhjhxxint}
\frac{\hbar^2}{2c^2m}\frac{\partial}{\partial
t}\left(\frac{\partial\psi_1}{\partial t}+\vec v(\vec
x,t)\cdot\nabla_{\vec x}\psi_1+\frac{i\sigma}{\hbar}\Psi_0(\vec
x,t)\psi_1\right)\\+\frac{\hbar^2}{2c^2m}\Div_{\vec x}\left\{\vec
v(\vec x,t)\left(\frac{\partial\psi_1}{\partial t}+\vec v(\vec
x,t)\cdot\nabla_{\vec x}\psi_1+\frac{i\sigma}{\hbar}\Psi_0(\vec
x,t)\psi_1\right)\right\}
\\
+\frac{\hbar^2}{2c^2m}\left(\frac{i\sigma}{\hbar}\Psi_0(\vec
x,t)\right)\left(\frac{\partial\psi_1}{\partial t}+\vec v(\vec
x,t)\cdot\nabla_{\vec x}\psi_1+\frac{i\sigma}{\hbar}\Psi_0(\vec
x,t)\psi_1\right)
\\
-i\hbar\left(\frac{\partial\psi_1}{\partial t}+\frac{1}{2}\vec
v(\vec x,t)\cdot\nabla_{\vec x}\psi_1+\frac{1}{2}\Div_{\vec
x}\left\{\psi_1\vec v(\vec
x,t)\right\}+\frac{i\sigma}{\hbar}\Psi_0(\vec
x,t)\psi_1\right)-\frac{\hbar^2}{2m}\Delta_{\vec
x}\psi_1-V\left(\vec x,t\right)\psi_1\\+\frac{
i\hbar\sigma}{2mc}div_{\vec x}\left\{\psi_1\vec A(\vec
x,t)\right\}+\frac{ i\hbar\sigma}{2mc}\vec A(\vec
x,t)\cdot\nabla_{\vec x}\psi_1+\left(
\frac{\sigma^2}{2mc^2}\left|\vec A(\vec
x,t)\right|^2\right)\psi_1=0,
\end{multline}
see subsection \ref{hggyugyuy1xx} for the details. Thus by
\er{vhfffngghkjgghfjjghghghhjghjgghkghgghjhggjjkgfgdiyfgfjkjgjgSYSPNxxjkhhkxxhjjhjhxxint}
in the non-relativistic limit we obtain:
\begin{multline}\label{vhfffngghkjgghfjjghghghhjghjgghkghgghjhggjjkgfgdiyfgfjkjgjgSYSPNjjjint}
i\hbar\left(\frac{\partial\psi_1}{\partial t}+\vec v(\vec
x,t)\cdot\nabla_{\vec
x}\psi_1\right)+\frac{i\hbar}{2}\left(div_{\vec x}\vec v(\vec
x,t)\right)\psi_1\approx -\frac{\hbar^2}{2m}\Delta_{\vec
x}\psi_1-V\left(\vec x,t\right)\psi_1+\frac{ i\hbar\sigma}{2mc}\vec
A(\vec x,t)\cdot\nabla_{\vec x}\psi_1\\+\frac{
i\hbar\sigma}{2mc}div_{\vec x}\left\{\psi_1\vec A(\vec
x,t)\right\}+\left(\sigma\Psi(\vec x,t)-\frac{\sigma}{c}\vec A(\vec
x,t)\cdot\vec v(\vec x,t)+\frac{\sigma^2}{2mc^2}\left|\vec A(\vec
x,t)\right|^2\right)\psi_1,
\end{multline}
that coincides with the Schr\"{o}dinger equation of the form
\er{vhfffngghkjgghfjjghghghhjghjgghkghgghjhggjjkgfgdiyfgfjkjgjgint}.

Next, again consider the motion of a quantum micro-particle having
inertial masses $m$ and the charge $\sigma$ with the known
gravitational and electromagnetical field with potentials $\vec
v(\vec x,t)$, $\vec A(\vec x,t)$ and $\Psi(\vec x,t)$ and additional
conservative field with potential $V(\vec x,t)$. Then consider a
Lagrangian density $L$ defined by
\begin{multline}\label{vhfffngghkjgghDDmmkkksmfggffgZZxxint}
L_2\left(\psi,\vec A,\Psi,\vec v,\vec
x,t\right):=-\frac{c^2m}{2}\psi\cdot\bar\psi-\frac{\hbar^2}{2m}\nabla_{\vec
x}\psi\cdot\nabla_{\vec
x}\bar\psi\\
+\frac{\hbar^2}{2c^2m}\left(\frac{\partial\psi}{\partial t}+\vec
v(\vec x,t)\cdot\nabla_{\vec x}\psi+\frac{i\sigma}{\hbar}\Psi_0(\vec
x,t)\psi\right)\left(\frac{\partial\bar\psi}{\partial t}+\vec v(\vec
x,t)\cdot\nabla_{\vec x}\bar\psi-\frac{i\sigma}{\hbar}\Psi_0(\vec
x,t)\bar\psi\right)\\
-\frac{\hbar\sigma i}{2mc}\left(\nabla_{\vec
x}\psi\cdot\bar\psi-\psi\cdot\nabla_{\vec x}\bar\psi\right)\cdot\vec
A(\vec x,t)-\frac{\sigma^2}{2mc^2}\left|\vec A(\vec
x,t)\right|^2\psi\cdot\bar\psi+V\left(\vec
x,t\right)\psi\cdot\bar\psi ,
\end{multline}
where $\psi(\vec x,t)\in\mathbb{C}$ is a wave function of the
system. Then, as before, it can be proven that $L$ is invariant
under the change of inertial or non-inertial cartesian coordinate
systems of the form \er{noninchgravortbstrghgggSYSPNxxint}, provided
that $\psi'=\psi$. We investigate stationary points of the
functional
\begin{equation}\label{btfffygtgyggyDDmmkkksmZZxxint}
J=\int_0^T\int_{\mathbb{R}^3}L\left(\psi,\vec A,\Psi,\vec v,\vec
x,t\right)d\vec xdt.
\end{equation}
Then,
\begin{multline}\label{vhfffngghkjgghDDmmkkkghgsmZZxxint}
0=\frac{\delta L_2}{\delta (\bar\psi)}=
-\frac{\hbar^2}{2c^2m}\frac{\partial}{\partial
t}\left(\frac{\partial\psi}{\partial t}+\vec v(\vec
x,t)\cdot\nabla_{\vec x}\psi+\frac{i\sigma}{\hbar}\Psi_0(\vec
x,t)\psi\right)\\-\frac{\hbar^2}{2c^2m}\Div_{\vec x}\left\{\vec
v(\vec x,t)\left(\frac{\partial\psi}{\partial t}+\vec v(\vec
x,t)\cdot\nabla_{\vec x}\psi+\frac{i\sigma}{\hbar}\Psi_0(\vec
x,t)\psi\right)\right\}
\\
-\frac{\hbar^2}{2c^2m}\left(\frac{i\sigma}{\hbar}\Psi_0(\vec
x,t)\right)\left(\frac{\partial\psi}{\partial t}+\vec v(\vec
x,t)\cdot\nabla_{\vec x}\psi+\frac{i\sigma}{\hbar}\Psi_0(\vec
x,t)\psi\right) -\frac{c^2m}{2}\psi+\frac{\hbar^2}{2m}\Delta_{\vec
x}\psi\\ -\frac{\hbar\sigma i}{2mc}\left(\vec A(\vec
x,t)\cdot\nabla_{\vec x}\psi+div_{\vec x}\left\{\psi\vec A(\vec
x,t)\right\}\right) -\frac{\sigma^2}{2mc^2}\left|\vec A(\vec
x,t)\right|^2\psi
+V\left(\vec x,t\right)\psi,
\end{multline}
and
\begin{multline}\label{vhfffngghkjgghDDmmkkkghguhyuysmZZxxint}
0=\frac{\delta L_2}{\delta (\psi)}=
-\frac{\hbar^2}{2c^2m}\frac{\partial}{\partial
t}\left(\frac{\partial\bar\psi}{\partial t}+\vec v(\vec
x,t)\cdot\nabla_{\vec x}\bar\psi+\frac{\bar
{(i)}\sigma}{\hbar}\Psi_0(\vec
x,t)\bar\psi\right)\\-\frac{\hbar^2}{2c^2m}\Div_{\vec x}\left\{\vec
v(\vec x,t)\left(\frac{\partial\bar\psi}{\partial t}+\vec v(\vec
x,t)\cdot\nabla_{\vec x}\bar\psi+\frac{\bar
{(i)}\sigma}{\hbar}\Psi_0(\vec x,t)\bar\psi\right)\right\}
\\
-\frac{\hbar^2}{2c^2m}\left(\frac{\bar
{(i)}\sigma}{\hbar}\Psi_0(\vec
x,t)\right)\left(\frac{\partial\bar\psi}{\partial t}+\vec v(\vec
x,t)\cdot\nabla_{\vec x}\bar\psi+\frac{\bar
{(i)}\sigma}{\hbar}\Psi_0(\vec x,t)\bar\psi\right)
-\frac{c^2m}{2}\bar\psi+\frac{\hbar^2}{2m}\Delta_{\vec x}\bar\psi
\\-\frac{\hbar\sigma \bar {(i)}}{2mc}\left(\vec A(\vec
x,t)\cdot\nabla_{\vec x}\bar\psi+div_{\vec x}\left\{\bar\psi\vec
A(\vec x,t)\right\}\right)-\frac{\sigma^2}{2mc^2}\left|\vec A(\vec
x,t)\right|^2\bar\psi
+V\left(\vec x,t\right)\bar\psi,
\end{multline}
where the last equality is just the complex conjugate of
\er{vhfffngghkjgghDDmmkkkghgsmZZxxint}. So we get that the
Euler-Lagrange equation for \er{btfffygtgyggyDDmmkkksmZZxxint}
coincides with the Klein--Gordon  equation of the form
\er{vhfffngghkjgghfjjghghghhjghjgghkghgghjhggjjkgfgdiyfgfjkjgjgSYSPNxxint}.
Next we would like to note that the Lagrangian density $L_2$,
defined by \er{vhfffngghkjgghDDmmkkksmfggffgZZxxint} obeys $U(1)$
local symmetry, i.e. for every scalar field $w:=w(\vec x,t)$ one can
easily deduce that $L_2$ in
\er{vhfffngghkjgghDDmmkkksmfggffgZZxxint} is invariant under the
transformation:
\begin{equation}\label{MaxVacFull1bjkgjhjhgjgjgkjfhjfdghghligioiuittrPPNhjkjhkjgghhjjhjintiiihhhjhjhjhint}
\begin{cases}
\psi\,\to\,e^{-\frac{i\sigma w}{c\hbar}}\,\psi
\\
\Psi\,\to\,\Psi+\frac{1}{c}\frac{\partial w}{\partial t}\\
\vec A\,\to\,\vec A-\nabla_{\vec x}w\\
\vec v\,\to\,\vec v\\
\Psi_0\,\to\,\Psi_0+\frac{1}{c}\left(\frac{\partial w}{\partial
t}+\vec v\cdot\nabla_{\vec x}w\right).
\end{cases}
\end{equation}
See subsection \ref{hggyugyuy1xx} for the generalization of all the
above to the system of $n$ relativistic-like spinless particles.

\subsection{Thermodynamics of a moving continuum medium} Again,
consistently with
\er{noninchgravortbstrjgghguittu2gjgghhjhghjhjgghgghghghtytythvfghfgghjggint},
consider in some cartesian coordinate system $(*)$ the second Law of
Newton for the moving continuum medium with the inertial mass
density $\mu$, the field of average (macroscopic) velocities $\vec
u$, the charge density $\rho$ and the electric current density $\vec
j$:
\begin{multline}\label{MaxVacFull1ninshtrgravortghhghgjkgghklhjgkghghjjkjhjkkggjkhjkhjjhhfhjhkjkhbbgjhzzzyyykkknnnNWBWHWPPNjjint}
\mu\left(\frac{\partial\vec u}{\partial t}+d_{\vec x}\vec u\cdot
\vec u\right)=\frac{\partial(\mu\vec u)}{\partial t}+\Div_{\vec
x}\left\{\mu\vec u\otimes\vec u\right\}=\\-\mu\vec u\times
curl_{\vec x}\vec v+\mu\left(\partial_{t}\vec
v+\nabla_{\vec x}\frac{1}{2}|\vec v|^2\right)+\rho\vec E+\frac{1}{c}\vec j\times\vec B+\Div_{\vec x}\mathcal{T}=\\
=-\mu(\vec u-\vec v)\times curl_{\vec x}\vec v
+\mu\left(\partial_t\vec v+ d_{\vec x}\vec v\cdot\vec
v\right)+\rho\vec E+\frac{1}{c}\vec j\times\vec B+\Div_{\vec
x}\mathcal{T}.
\end{multline}
Here $\rho\vec E+\frac{1}{c}\vec j\times\vec B$ is the volume
density of the Lorentz force where $\vec E$ and $\vec B$ are outer
electric and magnetic fields, assumed to be changing smoothly and
almost constant in the microscopic level, $\vec v$ is a vectorial
gravitational potential also assumed to be changing smoothly and
almost constant in the microscopic level, and $\mathcal{T}\in
\R^{3\times 3}$ is the symmetric Cauchy stress tensor of the
continuum medium. Moreover, the mass density $\mu$, clearly
satisfies the continuum equation:
\begin{equation}\label{MaxVacFull1ninshtrgravortghhghgjkgghklhjgkghghjjkjhjkkggjkhjkhjjhhfhjhkhjjkhjgzzzyyykkknnnNWBWHWPPNjjint}
\frac{\partial\mu}{\partial t}+div_{\vec x}\left(\mu\vec u\right)=0.
\end{equation}
In particular, multiplying
\er{MaxVacFull1ninshtrgravortghhghgjkgghklhjgkghghjjkjhjkkggjkhjkhjjhhfhjhkjkhbbgjhzzzyyykkknnnNWBWHWPPNjj}
by $\vec u$ and using
\er{MaxVacFull1ninshtrgravortghhghgjkgghklhjgkghghjjkjhjkkggjkhjkhjjhhfhjhkhjjkhjgzzzyyykkknnnNWBWHWPPNjj}
we deduce the equality of the balance of the kinetic energy:
\begin{multline}\label{MaxVacFull1ninshtrgravortghhghgjkgghklhjgkghghjjkjhjkkggjkhjkhjjhhfhjhkjkhbbgjhzzzyyykkknnnNWBWHWPPNjjnkjhbbjint}
\frac{\partial}{\partial t}\left(\frac{\mu}{2}\left|\vec
u\right|^2\right)+\Div_{\vec x}\left\{\left(\frac{\mu}{2}\left|\vec
u\right|^2\right)\vec u\right\}=\mu\left(\frac{\partial}{\partial
t}\left(\frac{1}{2}\left|\vec u\right|^2\right)+\vec
u\cdot\nabla_{\vec x}\left(\frac{1}{2}\left|\vec
u\right|^2\right)\right)=\\ \mu\left(\partial_{t}\vec v+\nabla_{\vec
x}\frac{1}{2}|\vec v|^2\right)\cdot\vec u+\left(\rho\vec
E+\frac{1}{c}\vec j\times\vec B\right)\cdot\vec u+\left(\Div_{\vec
x}\mathcal{T}\right)\cdot\vec u =\\ \mu\left(\partial_{t}\vec
v+\nabla_{\vec x}\frac{1}{2}|\vec v|^2\right)\cdot\vec u+\rho\vec
E\cdot\vec u-\frac{1}{c}\left(\vec u\times\vec B\right)\cdot\vec
j+\left(\Div_{\vec x}\mathcal{T}\right)\cdot\vec u.
\end{multline}

Next, it is well known (see \cite{SM}), that the First Law of
Thermodynamics of this moving medium has the following form:
\begin{multline}\label{MaxVacFull1ninshtrgravortghhghgjkgghklhjgkghghjjkjhjkkggjkhjkhjjhhfhjhkjkhbbgjhzzzyyykkknnnNWBWHWPPNjjhjhjhjhint}
\frac{\partial E}{\partial t}+\Div_{\vec x}\left\{E\vec
u\right\}=\mu\left(\frac{\partial}{\partial
t}\left(\frac{E}{\mu}\right)+\vec u\cdot\nabla_{\vec
x}\left(\frac{E}{\mu}\right)\right)\\ =\frac{1}{2}\left(d_{\vec
x}\vec u+\left\{d_{\vec x}\vec
u\right\}^T\right):\mathcal{T}-\Div_{\vec x}\vec q+\left(\vec
j-\rho\vec u\right)\cdot\left(\vec E+\frac{1}{c}\vec u\times\vec
B\right).
\end{multline}
Here $E$ is the volume density of the internal energy (energy per
unit volume) and consistently $\frac{E}{\mu}$ is the internal energy
per unit mass, $\vec q$ is the heat flux and
$$\left(\vec j-\rho\vec u\right)\cdot\left(\vec E+\frac{1}{c}\vec
u\times\vec B\right)$$ is the Joules heat term. In particular,
adding
\er{MaxVacFull1ninshtrgravortghhghgjkgghklhjgkghghjjkjhjkkggjkhjkhjjhhfhjhkjkhbbgjhzzzyyykkknnnNWBWHWPPNjjhjhjhjhint}
with
\er{MaxVacFull1ninshtrgravortghhghgjkgghklhjgkghghjjkjhjkkggjkhjkhjjhhfhjhkjkhbbgjhzzzyyykkknnnNWBWHWPPNjjnkjhbbjint}
and using the symmetry of $\mathcal{T}$, we deduce the following
equality of the balance of the energy:
\begin{multline}\label{MaxVacFull1ninshtrgravortghhghgjkgghklhjgkghghjjkjhjkkggjkhjkhjjhhfhjhkjkhbbgjhzzzyyykkknnnNWBWHWPPNjjnkjhbbjhjjhjhint}
\frac{\partial}{\partial t}\left(\frac{\mu}{2}\left|\vec
u\right|^2+E\right)+\Div_{\vec
x}\left\{\left(\frac{\mu}{2}\left|\vec u\right|^2+E\right)\vec
u\right\}=\mu\left(\frac{\partial}{\partial
t}\left(\frac{1}{2}\left|\vec u\right|^2+\frac{E}{\mu}\right)+\vec
u\cdot\nabla_{\vec x}\left(\frac{1}{2}\left|\vec
u\right|^2+\frac{E}{\mu}\right)\right)\\= \Div_{\vec
x}\left(\mathcal{T}\cdot\vec u\right)-\Div_{\vec x}\vec q+
\mu\left(\partial_{t}\vec v+\nabla_{\vec x}\frac{1}{2}|\vec
v|^2\right)\cdot\vec u+\vec E\cdot\vec j.
\end{multline}

Next the Second Law of Thermodynamics states that
\begin{equation}\label{MaxVacFull1ninshtrgravortghhghgjkgghklhjgkghghjjkjhjkkggjkhjkhjjhhfhjhkjkhbbgjhzzzyyykkknnnNWBWHWPPNjjhjhjhjhkpiint}
T\left(\frac{\partial \mathcal{S}}{\partial t}+\Div_{\vec
x}\left\{\mathcal{S}\vec
u\right\}\right)=T\mu\left(\frac{\partial}{\partial
t}\left(\frac{\mathcal{S}}{\mu}\right)+\vec u\cdot\nabla_{\vec
x}\left(\frac{\mathcal{S}}{\mu}\right)\right)\geq -\Div_{\vec x}\vec
q.
\end{equation}
Here $T:=T(\vec x,t)$ is the Kelvin temperature field and
$\mathcal{S}$ is the volume density of the entropy (entropy per unit
volume) and consistently $\frac{\mathcal{S}}{\mu}$ is the entropy
per unit mass. Moreover, we have the equality in
\er{MaxVacFull1ninshtrgravortghhghgjkgghklhjgkghghjjkjhjkkggjkhjkhjjhhfhjhkjkhbbgjhzzzyyykkknnnNWBWHWPPNjjhjhjhjhkpiint}
in the case of reversible or quasi-reversible process. In the latter
case we rewrite the First Law
\er{MaxVacFull1ninshtrgravortghhghgjkgghklhjgkghghjjkjhjkkggjkhjkhjjhhfhjhkjkhbbgjhzzzyyykkknnnNWBWHWPPNjjhjhjhjhint}
and the Second Law
\er{MaxVacFull1ninshtrgravortghhghgjkgghklhjgkghghjjkjhjkkggjkhjkhjjhhfhjhkjkhbbgjhzzzyyykkknnnNWBWHWPPNjjhjhjhjhkpiint}
together as:
\begin{multline}\label{MaxVacFull1ninshtrgravortghhghgjkgghklhjgkghghjjkjhjkkggjkhjkhjjhhfhjhkjkhbbgjhzzzyyykkknnnNWBWHWPPNjjhjhjhjhkloiojint}
\frac{\partial E}{\partial t}+\Div_{\vec x}\left\{E\vec u\right\}
=T\left(\frac{\partial \mathcal{S}}{\partial t}+\Div_{\vec
x}\left\{\mathcal{S}\vec u\right\}\right)+\frac{1}{2}\left(d_{\vec
x}\vec u+\left\{d_{\vec x}\vec
u\right\}^T\right):\mathcal{T}+\left(\vec j-\rho\vec
u\right)\cdot\left(\vec E+\frac{1}{c}\vec u\times\vec B\right)\\=
T\mu\left(\frac{\partial}{\partial
t}\left(\frac{\mathcal{S}}{\mu}\right)+\vec u\cdot\nabla_{\vec
x}\left(\frac{\mathcal{S}}{\mu}\right)\right)+
\frac{1}{2}\left(d_{\vec x}\vec u+\left\{d_{\vec x}\vec
u\right\}^T\right):\mathcal{T}+\left(\vec j-\rho\vec
u\right)\cdot\left(\vec E+\frac{1}{c}\vec u\times\vec B\right)
\\=\mu\left(\frac{\partial}{\partial
t}\left(\frac{E}{\mu}\right)+\vec u\cdot\nabla_{\vec
x}\left(\frac{E}{\mu}\right)\right) .
\end{multline}
In particular, if the stress tensor have the following particular
form
\begin{equation}\label{MaxVacFull1ninshtrgravortghhghgjkgghklhjgkghghjjkjhjkkggjkhjkhjjhhfhjhkhjjkhjgzzzyyykkknnnNWBWHWPPNjjhhjhjint}
\mathcal{T}=-p\,I,
\end{equation}
where $p$ is the scalar pressure and $I:=Id\in\R^{3\times 3}$ is the
identity matrix, then we rewrite the First Law of Thermodynamics
\er{MaxVacFull1ninshtrgravortghhghgjkgghklhjgkghghjjkjhjkkggjkhjkhjjhhfhjhkjkhbbgjhzzzyyykkknnnNWBWHWPPNjjhjhjhjhint}
as:
\begin{equation}\label{MaxVacFull1ninshtrgravortghhghgjkgghklhjgkghghjjkjhjkkggjkhjkhjjhhfhjhkjkhbbgjhzzzyyykkknnnNWBWHWPPNjjhjhjhjhinthjhjhjhint}
\mu\left(\frac{\partial}{\partial t}\left(\frac{E}{\mu}\right)+\vec
u\cdot\nabla_{\vec x}\left(\frac{E}{\mu}\right)\right)
=-\left(\Div_{\vec x}\vec u\right)p-\Div_{\vec x}\vec q+\left(\vec
j-\rho\vec u\right)\cdot\left(\vec E+\frac{1}{c}\vec u\times\vec
B\right),
\end{equation}
or equivalently as:
\begin{multline}\label{MaxVacFull1ninshtrgravortghhghgjkgghklhjgkghghjjkjhjkkggjkhjkhjjhhfhjhkjkhbbgjhzzzyyykkknnnNWBWHWPPNjjhjhjhjhlllkl;klint}
\frac{\partial}{\partial t}\left(\frac{E}{\mu}\right)+\vec
u\cdot\nabla_{\vec x}\left(\frac{E}{\mu}\right)
=-p\left(\frac{\partial}{\partial t}\left(\frac{1}{\mu}\right)+\vec
u\cdot\nabla_{\vec
x}\left(\frac{1}{\mu}\right)\right)-\frac{1}{\mu}\Div_{\vec x}\vec
q+\frac{1}{\mu}\left(\vec j-\rho\vec u\right)\cdot\left(\vec
E+\frac{1}{c}\vec u\times\vec B\right),
\end{multline}
and in the case of quasi-reversible process we rewrite
\er{MaxVacFull1ninshtrgravortghhghgjkgghklhjgkghghjjkjhjkkggjkhjkhjjhhfhjhkjkhbbgjhzzzyyykkknnnNWBWHWPPNjjhjhjhjhkloiojint}
as:
\begin{multline}\label{MaxVacFull1ninshtrgravortghhghgjkgghklhjgkghghjjkjhjkkggjkhjkhjjhhfhjhkjkhbbgjhzzzyyykkknnnNWBWHWPPNjjhjhjhjhlllklkljkkhhkgkint}
\frac{\partial}{\partial t}\left(\frac{E}{\mu}\right)+\vec
u\cdot\nabla_{\vec x}\left(\frac{E}{\mu}\right)
=-p\left(\frac{\partial}{\partial t}\left(\frac{1}{\mu}\right)+\vec
u\cdot\nabla_{\vec
x}\left(\frac{1}{\mu}\right)\right)+T\left(\frac{\partial}{\partial
t}\left(\frac{\mathcal{S}}{\mu}\right)+\vec u\cdot\nabla_{\vec
x}\left(\frac{\mathcal{S}}{\mu}\right)\right)\\+\frac{1}{\mu}\left(\vec
j-\rho\vec u\right)\cdot\left(\vec E+\frac{1}{c}\vec u\times\vec
B\right),
\end{multline}
where clearly $\frac{1}{\mu}$ is the volume per unit mass. Moreover,
the following inequality always holds:
\begin{equation}\label{MaxVacFull1ninshtrgravortghhghgjkgghklhjgkghghjjkjhjkkggjkhjkhjjhhfhjhkhjjkhjgzzzyyykkknnnNWBWHWPPNjjlkljjkkkjkjkint}
\vec q\cdot\nabla_{\vec x}T\leq 0.
\end{equation}
In particular, by inserting
\er{MaxVacFull1ninshtrgravortghhghgjkgghklhjgkghghjjkjhjkkggjkhjkhjjhhfhjhkhjjkhjgzzzyyykkknnnNWBWHWPPNjjlkljjkkkjkjkint}
into
\er{MaxVacFull1ninshtrgravortghhghgjkgghklhjgkghghjjkjhjkkggjkhjkhjjhhfhjhkjkhbbgjhzzzyyykkknnnNWBWHWPPNjjhjhjhjhkpiint}
we deduce
\begin{equation}\label{MaxVacFull1ninshtrgravortghhghgjkgghklhjgkghghjjkjhjkkggjkhjkhjjhhfhjhkjkhbbgjhzzzyyykkknnnNWBWHWPPNjjhjhjhjhkpihjjgjggjkhhhint}
\frac{\partial \mathcal{S}}{\partial t}+\Div_{\vec
x}\left\{\mathcal{S}\vec u\right\}=\mu\left(\frac{\partial}{\partial
t}\left(\frac{\mathcal{S}}{\mu}\right)+\vec u\cdot\nabla_{\vec
x}\left(\frac{\mathcal{S}}{\mu}\right)\right)\geq -\Div_{\vec
x}\left(\frac{1}{T}\,\vec q\right).
\end{equation}
Note that by
\er{MaxVacFull1ninshtrgravortghhghgjkgghklhjgkghghjjkjhjkkggjkhjkhjjhhfhjhkjkhbbgjhzzzyyykkknnnNWBWHWPPNjjhjhjhjhkpihjjgjggjkhhhint},
the total entropy of an arbitrary moving \underline{thermally
isolated} system is a non-decreasing function of time. Finally, we
remind the approximate Fourier's law (which is consistent with
\er{MaxVacFull1ninshtrgravortghhghgjkgghklhjgkghghjjkjhjkkggjkhjkhjjhhfhjhkhjjkhjgzzzyyykkknnnNWBWHWPPNjjlkljjkkkjkjkint}):
\begin{equation}\label{MaxVacFull1ninshtrgravortghhghgjkgghklhjgkghghjjkjhjkkggjkhjkhjjhhfhjhkhjjkhjgzzzyyykkknnnNWBWHWPPNjjlkljjint}
\vec q=-\chi\nabla_{\vec x}T,
\end{equation}
where $\chi$ is some positive material coefficient (not necessary a
constant).

Next consider the change of certain non-inertial cartesian
coordinate system $(*)$ to another cartesian coordinate system
$(**)$ of the form \er{noninchgravortbstrjgghguittu2int}. Then by
Proposition \ref{yghgjtgyrtrtint} we easily deduce that the Laws in
\er{MaxVacFull1ninshtrgravortghhghgjkgghklhjgkghghjjkjhjkkggjkhjkhjjhhfhjhkjkhbbgjhzzzyyykkknnnNWBWHWPPNjjhjhjhjhint},
\er{MaxVacFull1ninshtrgravortghhghgjkgghklhjgkghghjjkjhjkkggjkhjkhjjhhfhjhkjkhbbgjhzzzyyykkknnnNWBWHWPPNjjhjhjhjhkpiint},
\er{MaxVacFull1ninshtrgravortghhghgjkgghklhjgkghghjjkjhjkkggjkhjkhjjhhfhjhkjkhbbgjhzzzyyykkknnnNWBWHWPPNjjhjhjhjhkloiojint},
\er{MaxVacFull1ninshtrgravortghhghgjkgghklhjgkghghjjkjhjkkggjkhjkhjjhhfhjhkhjjkhjgzzzyyykkknnnNWBWHWPPNjjhhjhjint},
\er{MaxVacFull1ninshtrgravortghhghgjkgghklhjgkghghjjkjhjkkggjkhjkhjjhhfhjhkjkhbbgjhzzzyyykkknnnNWBWHWPPNjjhjhjhjhlllkl;klint},
\er{MaxVacFull1ninshtrgravortghhghgjkgghklhjgkghghjjkjhjkkggjkhjkhjjhhfhjhkjkhbbgjhzzzyyykkknnnNWBWHWPPNjjhjhjhjhlllklkljkkhhkgkint},
\er{MaxVacFull1ninshtrgravortghhghgjkgghklhjgkghghjjkjhjkkggjkhjkhjjhhfhjhkhjjkhjgzzzyyykkknnnNWBWHWPPNjjlkljjkkkjkjkint},
\er{MaxVacFull1ninshtrgravortghhghgjkgghklhjgkghghjjkjhjkkggjkhjkhjjhhfhjhkjkhbbgjhzzzyyykkknnnNWBWHWPPNjjhjhjhjhkpihjjgjggjkhhhint}
and
\er{MaxVacFull1ninshtrgravortghhghgjkgghklhjgkghghjjkjhjkkggjkhjkhjjhhfhjhkhjjkhjgzzzyyykkknnnNWBWHWPPNjjlkljjint}
are invariant under the change of a non-inertial cartesian
coordinate system given by \er{noninchgravortbstrjgghguittu2int},
provided that under \er{noninchgravortbstrjgghguittu2int} we have:
\begin{equation}\label{hjgjgjgkllujkint}
\begin{cases}\mu'=\mu,\\
E'=E,\\
\mathcal{S}'=\mathcal{S},\\
T'=T,\\
\vec q'=A(t)\cdot\vec q,\\
\mathcal{T}'=A(t)\cdot \mathcal{T}\cdot A^T(t)
\\
p'=p,\\
\chi'=\chi,\\
\rho'=\rho,\\
\vec u'=A(t)\cdot \vec u+\frac{dA}{dt}(t)\cdot\vec x+\frac{d\vec z}{dt}(t),\\
\vec v'=A(t)\cdot \vec v+\frac{dA}{dt}(t)\cdot\vec x+\frac{d\vec z}{dt}(t),\\
\vec j'=A(t)\cdot \vec j+\rho \frac{dA}{dt}(t)\cdot\vec
x+\rho\frac{d\vec z}{dt}(t).
\end{cases}
\end{equation}

See also section \ref{gjgggggggggggggggggggggggghkgu} for the
generalization of the Fourier's law
\er{MaxVacFull1ninshtrgravortghhghgjkgghklhjgkghghjjkjhjkkggjkhjkhjjhhfhjhkhjjkhjgzzzyyykkknnnNWBWHWPPNjjlkljjint}
to the case of anisotropic mediums.

\subsubsection{The case of an inviscid fluid/gas}
In the case of an inviscid fluid or gas equality
\er{MaxVacFull1ninshtrgravortghhghgjkgghklhjgkghghjjkjhjkkggjkhjkhjjhhfhjhkhjjkhjgzzzyyykkknnnNWBWHWPPNjjhhjhjint}
indeed holds. As a consequence, equality
\er{MaxVacFull1ninshtrgravortghhghgjkgghklhjgkghghjjkjhjkkggjkhjkhjjhhfhjhkjkhbbgjhzzzyyykkknnnNWBWHWPPNjjhjhjhjhlllkl;klint}
holds, and moreover, in the case of quasi-reversible process
equality
\er{MaxVacFull1ninshtrgravortghhghgjkgghklhjgkghghjjkjhjkkggjkhjkhjjhhfhjhkjkhbbgjhzzzyyykkknnnNWBWHWPPNjjhjhjhjhlllklkljkkhhkgkint}
also holds. Moreover, in the case of a classical ideal gas the
following state equality is well known:
\begin{equation}\label{noninchredPPNbjjhjjkhkhlkkkklint}
p\,=\,\frac{\mu}{m_0}\,k\,T,
\end{equation}
where $m_0$ is the mass of the single molecule of the given gas and
$k$ is the Boltzmann constant. Finally, for the ideal gas, we have
the following expression for the volume density of the internal
energy $E$:
\begin{equation}\label{noninchredPPNbjjhjjkhkhkljjljljklint}
E\,=\,\frac{\mu}{m_0}\,\gamma_0\,k\,T,
\end{equation}
where $\gamma_0>0$ is a constant that depends on the kind of the gas
(for the monatomic gas we have $\gamma_0=\frac{3}{2}$). On the other
hand, in the case of \underline{incompressible} fluid we have
\begin{equation}\label{noninchredPPNbjjhjjkhkhkljjljljklyjhfjint}
div_{\vec x}\vec u\,\equiv \,0,
\end{equation}
and the pressure $p$ is unspecified.

Then, as before, we easily deduce that the Laws in
\er{noninchredPPNbjjhjjkhkhlkkkklint},
\er{noninchredPPNbjjhjjkhkhkljjljljklint} and
\er{noninchredPPNbjjhjjkhkhkljjljljklyjhfjint} are invariant under
the change of a non-inertial cartesian coordinate system given by
\er{noninchgravortbstrjgghguittu2int}, provided that under
\er{noninchgravortbstrjgghguittu2int} we have \er{hjgjgjgkllujkint}.

\subsubsection{The case of the simplest viscous fluid/gas}
In the case of the simplest viscous fluid or gas we have the
following equality, that substitutes
\er{MaxVacFull1ninshtrgravortghhghgjkgghklhjgkghghjjkjhjkkggjkhjkhjjhhfhjhkhjjkhjgzzzyyykkknnnNWBWHWPPNjjhhjhjint}:
\begin{equation}\label{MaxVacFull1ninshtrgravortghhghgjkgghklhjgkghghjjkjhjkkggjkhjkhjjhhfhjhkhjjkhjgzzzyyykkknnnNWBWHWPPNjjhhjhjkpkl;kl;int}
\mathcal{T}=-p\,I+\left(\alpha\left(d_{\vec x}\vec u+\left\{d_{\vec
x}\vec u\right\}^T\right)+\beta\left(div_{\vec x}\vec
u\right)I\right),
\end{equation}
where $\alpha\geq 0$ and $\beta$ are some material coefficients.
Then, as in
\er{MaxVacFull1ninshtrgravortghhghgjkgghklhjgkghghjjkjhjkkggjkhjkhjjhhfhjhkjkhbbgjhzzzyyykkknnnNWBWHWPPNjjhjhjhjhlllkl;klint}
and
\er{MaxVacFull1ninshtrgravortghhghgjkgghklhjgkghghjjkjhjkkggjkhjkhjjhhfhjhkjkhbbgjhzzzyyykkknnnNWBWHWPPNjjhjhjhjhlllklkljkkhhkgkint},
by
\er{MaxVacFull1ninshtrgravortghhghgjkgghklhjgkghghjjkjhjkkggjkhjkhjjhhfhjhkhjjkhjgzzzyyykkknnnNWBWHWPPNjjhhjhjkpkl;kl;int}
we rewrite the First Law of Thermodynamics
\er{MaxVacFull1ninshtrgravortghhghgjkgghklhjgkghghjjkjhjkkggjkhjkhjjhhfhjhkjkhbbgjhzzzyyykkknnnNWBWHWPPNjjhjhjhjhint}
as:
\begin{multline}\label{MaxVacFull1ninshtrgravortghhghgjkgghklhjgkghghjjkjhjkkggjkhjkhjjhhfhjhkjkhbbgjhzzzyyykkknnnNWBWHWPPNjjhjhjhjhlllkl;kljm,n,int}
\frac{\partial}{\partial t}\left(\frac{E}{\mu}\right)+\vec
u\cdot\nabla_{\vec x}\left(\frac{E}{\mu}\right)
=\frac{\alpha}{2}\left|d_{\vec x}\vec u+\left\{d_{\vec x}\vec
u\right\}^T\right|^2+\frac{\beta}{2}\left|div_{\vec x}\vec
u\right|^2
\\-p\left(\frac{\partial}{\partial
t}\left(\frac{1}{\mu}\right)+\vec u\cdot\nabla_{\vec
x}\left(\frac{1}{\mu}\right)\right)-\frac{1}{\mu}\Div_{\vec x}\vec
q+\frac{1}{\mu}\left(\vec j-\rho\vec u\right)\cdot\left(\vec
E+\frac{1}{c}\vec u\times\vec B\right),
\end{multline}
and in the case of quasi-reversible process we rewrite
\er{MaxVacFull1ninshtrgravortghhghgjkgghklhjgkghghjjkjhjkkggjkhjkhjjhhfhjhkjkhbbgjhzzzyyykkknnnNWBWHWPPNjjhjhjhjhkloiojint}
as:
\begin{multline}\label{MaxVacFull1ninshtrgravortghhghgjkgghklhjgkghghjjkjhjkkggjkhjkhjjhhfhjhkjkhbbgjhzzzyyykkknnnNWBWHWPPNjjhjhjhjhlllklkljkkhhkgklkjjljljlint}
\frac{\partial}{\partial t}\left(\frac{E}{\mu}\right)+\vec
u\cdot\nabla_{\vec x}\left(\frac{E}{\mu}\right)
=\frac{\alpha}{2}\left|d_{\vec x}\vec u+\left\{d_{\vec x}\vec
u\right\}^T\right|^2+\frac{\beta}{2}\left|div_{\vec x}\vec
u\right|^2\\-p\left(\frac{\partial}{\partial
t}\left(\frac{1}{\mu}\right)+\vec u\cdot\nabla_{\vec
x}\left(\frac{1}{\mu}\right)\right)+T\left(\frac{\partial}{\partial
t}\left(\frac{\mathcal{S}}{\mu}\right)+\vec u\cdot\nabla_{\vec
x}\left(\frac{\mathcal{S}}{\mu}\right)\right)+\frac{1}{\mu}\left(\vec
j-\rho\vec u\right)\cdot\left(\vec E+\frac{1}{c}\vec u\times\vec
B\right).
\end{multline}
Moreover, in the case of a classical ideal gas the equalities
\er{noninchredPPNbjjhjjkhkhlkkkklint} and
\er{noninchredPPNbjjhjjkhkhkljjljljklint} are still valid in the
case of viscous flow. On the other hand, in the case of
\underline{incompressible} viscous fluid we have
\er{noninchredPPNbjjhjjkhkhkljjljljklyjhfjint} and the pressure $p$
is also unspecified.

Then, as before, by Proposition \ref{yghgjtgyrtrtint} we easily
deduce that the Laws in
\er{MaxVacFull1ninshtrgravortghhghgjkgghklhjgkghghjjkjhjkkggjkhjkhjjhhfhjhkhjjkhjgzzzyyykkknnnNWBWHWPPNjjhhjhjkpkl;kl;int},
\er{MaxVacFull1ninshtrgravortghhghgjkgghklhjgkghghjjkjhjkkggjkhjkhjjhhfhjhkjkhbbgjhzzzyyykkknnnNWBWHWPPNjjhjhjhjhlllkl;kljm,n,int},
\er{MaxVacFull1ninshtrgravortghhghgjkgghklhjgkghghjjkjhjkkggjkhjkhjjhhfhjhkjkhbbgjhzzzyyykkknnnNWBWHWPPNjjhjhjhjhlllklkljkkhhkgklkjjljljlint},
\er{noninchredPPNbjjhjjkhkhlkkkklint},
\er{noninchredPPNbjjhjjkhkhkljjljljklint} and
\er{noninchredPPNbjjhjjkhkhkljjljljklyjhfjint} are invariant under
the change of a non-inertial cartesian coordinate system given by
\er{noninchgravortbstrjgghguittu2int}, provided that under
\er{noninchgravortbstrjgghguittu2int} we have \er{hjgjgjgkllujkint}.

See also the end of subsection \ref{jghgfhhfhfkhhj} for the
generalization of the viscosity law
\er{MaxVacFull1ninshtrgravortghhghgjkgghklhjgkghghjjkjhjkkggjkhjkhjjhhfhjhkhjjkhjgzzzyyykkknnnNWBWHWPPNjjhhjhjkpkl;kl;int}
to the case of an anisotropic Newtonian fluid.

\subsubsection{Lagrangian coordinates and the simplest models of
elastic bodies}\label{hjgyujtgjtututtuint} In some cartesian
coordinate system consider a motion of some continuum medium
occupying a region $\Omega\subset\mathbb{R}^3$ at some fixed instant
of time $t=t_0$ and having the velocity field $\vec u:=\vec u(\vec
x,t)$. Next let $\vec r(t,\vec
y):\mathbb{R}\times\Omega\to\mathbb{R}^3$ be a solution of the
following initial value problem for an ordinary differential
equation:
\begin{equation}\label{hjgghhghhffhkkint}
\begin{cases}
\frac{\partial\vec r}{\partial t}(t,\vec y)=\vec u\left(\vec
r(t,\vec y),t\right)\quad\quad\forall t\in\mathbb{R},\;\;\forall\vec
y\in\Omega\\
\vec r(t_0,\vec y)=\vec y\quad\quad\forall\vec y\in\Omega.
\end{cases}
\end{equation}
Then, clearly $\vec x=\vec r(t,\vec y)$ stands for the spatial
coordinates at the instant of time $t$ of the parcel of continuum,
having initial coordinates $\vec y$. Moreover,
we can deduce that $det\left\{d_{\vec y}\vec r(t,\vec y)\right\}\neq
0$ for every instant of time $t$ and so, for the given instant of
time $t$ the mapping $\vec x=\vec r(t,\vec y)$ is locally invertible
i.e. the equation $\vec x=\vec r(t,\vec y)$ can be resolved in $\vec
y$ (see subsection \ref{hjgyujtgjtututtu} for details). Thus there
exists a regular mapping $\vec y=\vec f(\vec x,t)$, such that
\begin{equation}\label{hjgghhghhffhkkihiyjhjkhjhjhint}
\vec f\left(\vec r(t,\vec y),t\right)=\vec y\quad\forall \vec
y\in\Omega\quad\quad\text{and}\quad\quad\vec r\left(t,\vec f(\vec
x,t)\right)=\vec x.
\end{equation}
Then clearly $\vec y=\vec f(\vec x,t)$ stands for the initial
coordinates of the parcel of continuum, having coordinates $\vec x$
at the instant of time $t$. Next differentiating the first equation
in \er{hjgghhghhffhkkihiyjhjkhjhjhint} by $t$ due to the chain rule
and using \er{hjgghhghhffhkkint} we deduce:
\begin{equation}\label{hjgghhghhffhkkmkljjnnnnint}
\begin{cases}
\frac{\partial\vec f}{\partial t}\left(\vec x,t\right)+d_{\vec
x}\vec f\left(\vec x,t\right)\cdot\vec u\left(\vec x,t\right)=0
\\
\vec f(\vec x,t_0)=\vec x.
\end{cases}
\end{equation}
The \underline{cartesian} coordinates $(\vec x,t)$ are called the
Eulerian coordinates of the continuum medium. In contrast, change of
variables
\begin{equation}\label{hjgghhghhffhkkmkljjnnnnoujkuint}
\begin{cases}
t=t\\
\vec y=\vec f\left(\vec x,t\right)
\end{cases}
\end{equation}
to the new coordinates $(\vec y,t)$ of the space-time leads to
generally non-cartesian \underline{curvilinear} coordinates that
called the Lagrangian or the reference coordinates of the continuum
medium. Next note that in the case of a given scalar field in the
Eulerian coordinates $\Theta=\Theta(\vec x,t)$ and the corresponding
field in Lagrangian coordinates $\Theta_1(\vec
y,t):=\Theta\left(\vec r(t,\vec y),t\right)$, due to the chain rule
and using \er{hjgghhghhffhkkint} we obtain that for
$\frac{\partial\Theta_1}{\partial t}(\vec y,t)$ in the Lagrangian
coordinates corresponds the following expression in the Eulerian
coordinates:
\begin{equation}\label{hjgghhghhffhkkmkljjnnnnoujkujjjjljjint}
\frac{\partial\Theta_1}{\partial t}(\vec
y,t)=\frac{\partial\Theta}{\partial t}\left(\vec x,t\right)+\vec
u\left(\vec x,t\right)\cdot\nabla_{\vec x}\Theta\left(\vec
x,t\right).
\end{equation}
Thus by \er{hjgghhghhffhkkmkljjnnnnoujkujjjjljjint} we also obtain
that in the case of a given vector field in the Eulerian coordinates
$\vec g=\vec g(\vec x,t)$ and the corresponding field in Lagrangian
coordinates $\vec g_1(\vec y,t):=\vec g\left(\vec x,t\right)=\vec
g\left(\vec r(t,\vec y),t\right)$ for $\frac{\partial\vec
g_1}{\partial t}(\vec y,t)$ in the Lagrangian coordinates
corresponds the following expression in the Eulerian coordinates:
\begin{equation}\label{hjgghhghhffhkkmkljjnnnnoujkujjjjljjkhhjkjkint}
\frac{\partial\vec g_1}{\partial t}(\vec y,t)=\frac{\partial\vec
g}{\partial t}\left(\vec x,t\right)+ d_{\vec x}\vec g\left(\vec
x,t\right)\cdot \vec u\left(\vec x,t\right).
\end{equation}

Next, as before, assume that the change of some non-inertial
cartesian system $(*)$ of Eulerian coordinates to another cartesian
system $(**)$ of Eulerian coordinates is of the form:
\begin{equation}\label{jjjjkkkkint}
\begin{cases}
\vec x'=A(t)\cdot\vec x+\vec z(t),\\
t'=t,
\end{cases}
\end{equation}
where $A(t)\in SO(3)$. Then in subsection \ref{hjgyujtgjtututtu} we
deduce that the law of transformation of the Lagrangian coordinates
$(\vec y,t)\to(\vec y',t')$, consistent with \er{jjjjkkkkint}, is
the following:
\begin{equation}\label{jjjjkkkkkhhhhjhjhint}
\begin{cases}
\vec y'=A(t_0)\cdot\vec y+\vec z(t_0),\\
t'=t,
\end{cases}
\end{equation}
i.e. if we define $\vec r'(t',\vec
y'):\mathbb{R}\times\mathbb{R}^3\to\mathbb{R}^3$ as:
\begin{equation}\label{hjgjgjgkllujkjjjkhhjhint}
\vec r'(t',\vec y')=\vec r'\left(t,A(t_0)\cdot\vec y+\vec
z(t_0)\right):=A(t)\cdot\vec r(t,\vec y)+\vec z(t),
\end{equation}
then consistently with \er{hjgghhghhffhkkint} we have
\begin{equation}\label{hjgghhghhffhkklkjljjkjkkjlint}
\begin{cases}
\frac{\partial\vec r'}{\partial t'}(t',\vec y')=\vec u'\left(\vec
r'(t',\vec y'),t'\right)\\
\vec r'(t_0,\vec y')=\vec y'.
\end{cases}
\end{equation}
Moreover, for inverse mappings we have
\begin{equation}\label{hjgjgjgkllujkjjjkhhjhljjkint}
\vec f'(\vec x',t')=\vec f'\left(A(t)\cdot\vec x+\vec
z(t),t\right)=A(t_0)\cdot\vec f(\vec x,t)+\vec z(t_0).
\end{equation}

Next it is well known that the finite strain tensor
$\mathcal{E}(\vec x,t)\in\mathbb{R}^{3\times 3}$ of an elastic
continuum medium in the cartesian Eulerian coordinates $(\vec x,t)$
has the form:
\begin{equation}\label{hjgjgjgkllujkjjjkhhjhljjkjkljjljjint}
\mathcal{E}(\vec x,t):=\frac{1}{2}\left(I-\left\{d_{\vec x}\vec
f(\vec x,t)\right\}^T\cdot d_{\vec x}\vec f(\vec x,t)\right),
\end{equation}
where the mapping $\vec f(\vec x,t)$ is given by
\er{hjgghhghhffhkkihiyjhjkhjhjhint} and satisfies
\er{hjgghhghhffhkkmkljjnnnnint}. Then by \er{jjjjkkkkint} and
\er{hjgjgjgkllujkjjjkhhjhljjkint} we deduce that under
\er{jjjjkkkkint} the matrix valued field $\mathcal{E}$ transforms
as:
\begin{equation}\label{hjgjgjgkllujkjjjkhhjhljjkjkljjljjlljjkhhjint}
\mathcal{E}'(\vec x',t)=A(t)\cdot\mathcal{E}(\vec x,t)\cdot A^T(t),
\end{equation}
i.e. $\mathcal{E}$ is a proper matrix field as it was defined in
Definition \ref{bggghghgjint}.
\begin{remark}\label{uggghfhfhfgint}
It is quite clear that the finite strain tensor, defined by
\er{hjgjgjgkllujkjjjkhhjhljjkjkljjljjint} depends essentially on the
choice of initial instant of time $t_0$. Thus, the initial time
$t_0$ for \er{hjgjgjgkllujkjjjkhhjhljjkjkljjljjint} is always chosen
(possibly fictitiously) in such a way that
our elastic body is relaxed at time $t_0$.
\end{remark}
Next, in the case of the simplest elastic body we have the following
Hooke's law that is similar to
\er{MaxVacFull1ninshtrgravortghhghgjkgghklhjgkghghjjkjhjkkggjkhjkhjjhhfhjhkhjjkhjgzzzyyykkknnnNWBWHWPPNjjhhjhjkpkl;kl;int}:
\begin{equation}\label{MaxVacFull1ninshtrgravortghhghgjkgghklhjgkghghjjkjhjkkggjkhjkhjjhhfhjhkhjjkhjgzzzyyykkknnnNWBWHWPPNjjhhjhjkphkhkhhklklklint}
\mathcal{T}(\vec x,t)=\left(\alpha\mathcal{E}(\vec
x,t)+\beta\left(tr\left\{\mathcal{E}(\vec
x,t)\right\}\right)I\right),
\end{equation}
where $\alpha$ and $\beta$ are some material coefficients, which are
not necessary constant. Here $\mathcal{T}(\vec x,t)$ is the Cauchy
stress tensor appearing in the equations of the motion of the medium
\er{MaxVacFull1ninshtrgravortghhghgjkgghklhjgkghghjjkjhjkkggjkhjkhjjhhfhjhkjkhbbgjhzzzyyykkknnnNWBWHWPPNjjint}
in the Eulerian coordinates $(\vec x,t)$ and $\mathcal{E}(\vec x,t)$
is the finite strain tensor defined by
\er{hjgjgjgkllujkjjjkhhjhljjkjkljjljjint}. Then by
\er{hjgjgjgkllujkjjjkhhjhljjkjkljjljjlljjkhhjint} we easily deduce
that the Hooke's law
\er{MaxVacFull1ninshtrgravortghhghgjkgghklhjgkghghjjkjhjkkggjkhjkhjjhhfhjhkhjjkhjgzzzyyykkknnnNWBWHWPPNjjhhjhjkphkhkhhklklklint}
is invariant under the change of a non-inertial cartesian coordinate
system given by \er{jjjjkkkkint}, provided that, as usual, under
\er{jjjjkkkkint} we have:
\begin{equation}\label{hjgjgjgkllujkjkhkkkjklint}
\begin{cases}
\alpha'=\alpha,\\
\beta'=\beta,\\
\mathcal{T}'=A(t)\cdot \mathcal{T}\cdot A^T(t).
\end{cases}
\end{equation}
See the end of subsection \ref{hjgujhghfhffgfg} for the
generalization of the Hooke's law
\er{MaxVacFull1ninshtrgravortghhghgjkgghklhjgkghghjjkjhjkkggjkhjkhjjhhfhjhkhjjkhjgzzzyyykkknnnNWBWHWPPNjjhhjhjkphkhkhhklklklint}
to the case of anisotropic elastic bodies.

\subsubsection{The equations of the motion of the general continuum medium in Lagrangian coordinates}
The equation of the motion of the continuum in the cartesian
Eulerian coordinates $(\vec x,t)$ has the form
\er{MaxVacFull1ninshtrgravortghhghgjkgghklhjgkghghjjkjhjkkggjkhjkhjjhhfhjhkjkhbbgjhzzzyyykkknnnNWBWHWPPNjjint},
i.e:
\begin{multline}\label{MaxVacFull1ninshtrgravortghhghgjkgghklhjgkghghjjkjhjkkggjkhjkhjjhhfhjhkjkhbbgjhzzzyyykkknnnNWBWHWPPNjjjljkjkint}
\mu\left(\frac{\partial\vec u}{\partial t}+d_{\vec x}\vec u\cdot
\vec u\right)=-\mu\vec u\times curl_{\vec x}\vec
v+\mu\left(\partial_{t}\vec
v+\nabla_{\vec x}\frac{1}{2}|\vec v|^2\right)+\rho\vec E+\frac{1}{c}\vec j\times\vec B+\Div_{\vec x}\mathcal{T}=\\
=-\mu(\vec u-\vec v)\times curl_{\vec x}\vec v
+\mu\left(\partial_t\vec v+ d_{\vec x}\vec v\cdot\vec
v\right)+\rho\vec E+\frac{1}{c}\vec j\times\vec B+\Div_{\vec
x}\mathcal{T},
\end{multline}
and the continuum equation has the form
\er{MaxVacFull1ninshtrgravortghhghgjkgghklhjgkghghjjkjhjkkggjkhjkhjjhhfhjhkhjjkhjgzzzyyykkknnnNWBWHWPPNjjint},
i.e:
\begin{equation}\label{MaxVacFull1ninshtrgravortghhghgjkgghklhjgkghghjjkjhjkkggjkhjkhjjhhfhjhkhjjkhjgzzzyyykkknnnNWBWHWPPNjjjkkint}
\frac{\partial\mu}{\partial t}+div_{\vec x}\left(\mu\vec u\right)=0.
\end{equation}
We would like to get the analogues of these equations in the
Lagrangian coordinates $(\vec y,t)$. First of all we can obtain:
\begin{equation}\label{hjhjjttuttutukjhkhlkkjkhint}
\mu\left(\vec r(t,\vec y),t\right)=\mu(\vec y,t_0)\left|det\{d_{\vec
y}\vec r(t,\vec y)\}\right|^{-1}\quad\quad\forall t,\;\forall\vec
y\in\Omega.
\end{equation}
Then equality \er{hjhjjttuttutukjhkhlkkjkhint} substitutes the
continuum equation
\er{MaxVacFull1ninshtrgravortghhghgjkgghklhjgkghghjjkjhjkkggjkhjkhjjhhfhjhkhjjkhjgzzzyyykkknnnNWBWHWPPNjjjkkint}
in the Lagrangian coordinates $(\vec y,t)$ (see subsection
\ref{hjgyujtgjtututtu} for details). In particular, by
\er{hjgghhghhffhkkihiyjhjkhjhjhint},
\er{hjhjjttuttutukjhkhlkkjkhint} in Eulerian coordinates reeds as
\begin{equation}\label{hjhjjttuttutukjhkhlkkjkhjljjzzint}
\mu\left(\vec x,t\right)=\mu(\vec x,t_0)\left|det\{d_{\vec x}\vec
f(\vec x,t)\}\right|.
\end{equation}

Next, the following Calculus fact is well known in Continuum
Mechanics: if, given some matrix valued field
\begin{equation}\label{hkhkhhhnnmhhhjklklint}
\mathcal{T}(\vec x,t)\in\mathbb{R}^{3\times 3},
\end{equation}
we denote
\begin{equation}\label{Mgjgjghghklklint}
F(\vec y,t):=d_{\vec y}\vec r(t,\vec y)\in\mathbb{R}^{3\times 3},
\end{equation}
and
\begin{equation}\label{Mgjgjghghkjljlljljjjkjhhhjhint}
\mathcal{R}(\vec y,t):=\mathcal{T}\left(\vec r(t,\vec
y),t\right)\cdot \{F^{-1}(\vec y,t)\}^T\left(det\,F(\vec
y,t)\right),
\end{equation}
then we must have:
\begin{equation}\label{hjkgugfhfhffjkjkojljhjjgjjjggjint}
\Div_{\vec y}\mathcal{R}(\vec y,t)=\left( \Div_{\vec
x}\mathcal{T}\left(\vec r(t,\vec y),t\right)\right)\left(det\,F(\vec
y,t)\right).
\end{equation}
In particular, if $\mathcal{T}(\vec x,t)$ is the Cauchy stress
tensor, then $\mathcal{R}(\vec y,t)$ defined by
\er{Mgjgjghghkjljlljljjjkjhhhjhint} is called the first
Piola--Kirchhoff stress tensor. Then by
\er{hjkgugfhfhffjkjkojljhjjgjjjggjint},
\er{hjhjjttuttutukjhkhlkkjkhint}, \er{hjgghhghhffhkkint} and
\er{hjgghhghhffhkkmkljjnnnnoujkujjjjljjkhhjkjkint} we finally
rewrite
\er{MaxVacFull1ninshtrgravortghhghgjkgghklhjgkghghjjkjhjkkggjkhjkhjjhhfhjhkjkhbbgjhzzzyyykkknnnNWBWHWPPNjjjljkjkint}
in Lagrangian coordinates as:
\begin{multline}\label{MaxVacFull1ninshtrgravortghhghgjkgghklhjhgjggjgkghghjjkjhjkkggjkhjkhjjhhfhjhkjkhbbgjhzzzyyykkknnnNWBWHWPPNjjjljkjkint}
\mu(\vec y,t_0)\frac{\partial^2\vec r}{\partial t^2}(t,\vec
y)=\Div_{\vec y}\mathcal{R}(\vec y,t)\\-\mu(\vec
y,t_0)\frac{\partial\vec r}{\partial t}(t,\vec y)\times curl_{\vec
x}\vec v\left(\vec r(t,\vec y),t\right)+\mu(\vec
y,t_0)\left(\partial_{t}\vec v\left(\vec r(t,\vec
y),t\right)+\nabla_{\vec x}\frac{1}{2}\left|\vec v\left(\vec
r(t,\vec y),t\right)\right|^2\right)\\+\frac{\mu(\vec
y,t_0)}{\mu\left(\vec r(t,\vec y),t\right)}\left(\rho\left(\vec
r(t,\vec y),t\right)\vec E\left(\vec r(t,\vec
y),t\right)+\frac{1}{c}\vec j\left(\vec r(t,\vec
y),t\right)\times\vec B\left(\vec r(t,\vec y),t\right)\right).
\end{multline}

\subsubsection{Propagation of sound wave in the moving inviscid gas
or fluid}\label{jihggghjgjggint} Since in the case of inviscid fluid
or gas
\er{MaxVacFull1ninshtrgravortghhghgjkgghklhjgkghghjjkjhjkkggjkhjkhjjhhfhjhkhjjkhjgzzzyyykkknnnNWBWHWPPNjjhhjhjint}
holds:
\begin{equation}\label{MaxVacFull1ninshtrgra,mnvortghhghgjkgghklhjgkghghjjkjhjkkggjkhjkhjjhhfhjhkjkhbbgjhzzzyyykkknnnNWBWHWPPNjjhjhjhjhplllhjggjlint}
\mathcal{T}=-p I,
\end{equation}
consistently with
\er{MaxVacFull1ninshtrgravortghhghgjkgghklhjgkghghjjkjhjkkggjkhjkhjjhhfhjhkjkhbbgjhzzzyyykkknnnNWBWHWPPNjjint},
we rewrite the second Law of Newton for the our inviscid  continuum
medium with the inertial mass density $\mu$, the field of average
(macroscopic) velocities $\vec u$, the charge density $\rho$ and the
electric current density $\vec j$:
\begin{multline}\label{MaxVacFull1ninshtrgravortghhghgjkgghklhjgkghghjjkjhjkkggjkhjkhjjhhfhjhkjkhbbgjhzzzyyykkknnnNWBWHWPPNjjhuhhint}
\mu\left(\frac{\partial\vec u}{\partial t}+d_{\vec x}\vec u\cdot
\vec u\right)=\frac{\partial}{\partial t}(\mu\vec u)+\Div_{\vec
x}\left\{\mu\vec u\otimes\vec u\right\}=\\-\mu\vec u\times
curl_{\vec x}\vec v+\mu\left(\partial_{t}\vec
v+\nabla_{\vec x}\frac{1}{2}|\vec v|^2\right)+\rho\vec E+\frac{1}{c}\vec j\times\vec B-\nabla_{\vec x}p=\\
=-\mu(\vec u-\vec v)\times curl_{\vec x}\vec v
+\mu\left(\partial_t\vec v+ d_{\vec x}\vec v\cdot\vec
v\right)+\rho\vec E+\frac{1}{c}\vec j\times\vec B-\nabla_{\vec x}p,
\end{multline}
Here, as before, $\rho\vec E+\frac{1}{c}\vec j\times\vec B$ is the
volume density of the Lorentz force where $\vec E$ and $\vec B$ are
outer electric and magnetic fields, assumed to be changing smoothly
and almost constant in the microscopic level, $\vec v$ is a
vectorial gravitational potential also assumed to be changing
smoothly and almost constant in the microscopic level, and $p$ is
the pressure. Moreover, the mass density $\mu$, clearly satisfies
the continuum equation
\er{MaxVacFull1ninshtrgravortghhghgjkgghklhjgkghghjjkjhjkkggjkhjkhjjhhfhjhkhjjkhjgzzzyyykkknnnNWBWHWPPNjjint}:
\begin{equation}\label{MaxVacFull1ninshtrgravortghhghgjkgghklhjgkghghjjkjhjkkggjkhjkhjjhhfhjhkhjjkhjgzzzyyykkknnnNWBWHWPPNjj1int}
\frac{\partial\mu}{\partial t}+div_{\vec x}\left(\mu\vec u\right)=0.
\end{equation}
Next, by
\er{MaxVacFull1ninshtrgravortghhghgjkgghklhjgkghghjjkjhjkkggjkhjkhjjhhfhjhkjkhbbgjhzzzyyykkknnnNWBWHWPPNjjhjhjhjhinthjhjhjhint}
the First Law of Thermodynamics of this moving medium has the
following form:
\begin{multline}\label{MaxVacFull1ninshtrgravortghhghgjkgghklhjgkghghjjkjhjkkggjkhjkhjjhhfhjhkjkhbbgjhzzzyyykkknnnNWBWHWPPNjjhjhjhjhugyyyint}
\frac{\partial E}{\partial t}+\Div_{\vec x}\left\{E\vec
u\right\}=\mu\left(\frac{\partial}{\partial
t}\left(\frac{E}{\mu}\right)+\vec u\cdot\nabla_{\vec
x}\left(\frac{E}{\mu}\right)\right)\\ =-\left(\Div_{\vec x}\vec
u\right)p-\Div_{\vec x}\vec q+\left(\vec j-\rho\vec
u\right)\cdot\left(\vec E+\frac{1}{c}\vec u\times\vec B\right).
\end{multline}
Here $E$ is the volume density of the internal energy (energy per
unit volume) and consistently $\frac{E}{\mu}$ is the internal energy
per unit mass and $\vec q$ is the heat flux. Moreover, the internal
energy per unit mass is a function of the density $\mu$, pressure
$p$ and Kelvin temperature $T$:
\begin{equation}\label{MaxVacFull1ninshtrgravortghhghgjkgghklhjgkghghjjkjhjkkggjkhjkhjjhhfhjhkjkhbbgjhzzzyyykkknnnNWBWHWPPNjjhjhjhjhplllhjggjlhint}
\frac{E}{\mu}=U\left(\mu,p,T\right).
\end{equation}
There is also a state equation of the form:
\begin{equation}\label{MaxVacFull1ninsmljjlhtrgravortghhghgjkgghklhjgkghghjjkjhjkkggjkhjkhjjhhfhjhkjkhbbgjhzzzyyykkknnnNWBWHWPPNjjhjhjhjhplllhjggjlhint}
T=g\left(\mu,p\right).
\end{equation}
We remind that in the case of the simplest ideal gas
\er{MaxVacFull1ninshtrgravortghhghgjkgghklhjgkghghjjkjhjkkggjkhjkhjjhhfhjhkjkhbbgjhzzzyyykkknnnNWBWHWPPNjjhjhjhjhplllhjggjlhint}
takes particular form of \er{noninchredPPNbjjhjjkhkhkljjljljklint}:
\begin{equation}\label{noninchredPPNbjjhjjkhkhkljjljljklhjhjhjghhgghjhhjint}
\frac{E}{\mu}\,=\,\frac{\gamma_0}{m_0}\,\,k\,T,
\end{equation}
and
\er{MaxVacFull1ninsmljjlhtrgravortghhghgjkgghklhjgkghghjjkjhjkkggjkhjkhjjhhfhjhkjkhbbgjhzzzyyykkknnnNWBWHWPPNjjhjhjhjhplllhjggjlhint}
takes particular form of \er{noninchredPPNbjjhjjkhkhlkkkklint}:
\begin{equation}\label{noninchredPPNbjjhjjkhkhlkkkklhjgghghint}
T\,=\,\frac{m_0}{k}\,\frac{p}{\mu},
\end{equation}
where $m_0$ is the mass of the single molecule of the given gas and
$k$ is the Boltzmann constant and $\gamma_0>0$ is a constant that
depends on the kind of the gas (for the monatomic gas we have
$\gamma_0=\frac{3}{2}$). So, by inserting
\er{MaxVacFull1ninsmljjlhtrgravortghhghgjkgghklhjgkghghjjkjhjkkggjkhjkhjjhhfhjhkjkhbbgjhzzzyyykkknnnNWBWHWPPNjjhjhjhjhplllhjggjlhint}
into
\er{MaxVacFull1ninshtrgravortghhghgjkgghklhjgkghghjjkjhjkkggjkhjkhjjhhfhjhkjkhbbgjhzzzyyykkknnnNWBWHWPPNjjhjhjhjhplllhjggjlhint}
we have:
\begin{equation}\label{MaxVacFull1nkkkinshtrgravortghhghgjkgghklhjgkghghjjkjhjkkggjkhjkhjjhhfhjhkjkhbbgjhzzzyyykkknnhjhhnNWBWHWPPNjjhjhjhjhplllhjggjll;;lijiihint}
\frac{E}{\mu}
=U\left(\mu,p,g\left(\mu,p\right)\right):=F\left(\mu,p\right).
\end{equation}
In particular, in the case where
\er{noninchredPPNbjjhjjkhkhkljjljljklhjhjhjghhgghjhhjint} and
\er{noninchredPPNbjjhjjkhkhlkkkklhjgghghint} are valid we have
\begin{equation}\label{MaxVacFull1nkkkinshtrgravortghhghgjkgghklhjgkghghjjkjhjkkggjkhjkhjjhhfhjhkjkhbbgjhzzzyyykkknnhjhhnNWBWHWPPNjjhjhjhjhplllhjggjll;;lijiihkhhkgjgjint}
\frac{E}{\mu}=F\left(\mu,p\right):=\gamma_0\,\frac{p}{\mu}.
\end{equation}

Next we assume that
\begin{equation}\label{MaxVacFull1ninshtrgravortghhghgjkgghklhjgkghghjjkjhjkkggjkhjkhjjhhfhjhkjkhbbgjhzzzyyykkknnnNWBWHWPPNjjhjhjhjhplllhjggjlint}
\vec u=\vec u_0+\vec u_1,\quad\mu=\mu_0+\mu_1,\quad p=p_0+p_1.
\end{equation}
where $\vec u_0(\vec x,t),\mu_0(\vec x,t),p_0(\vec x,t)$ are the
averages of $\vec u(\vec x,t),\mu(\vec x,t),p(\vec x,t)$ on small
spatial and temporal intervals, surrounding the point $(\vec x,t)$.
Although we assume these intervals of space and time to be very
small, we also assume them to be quite \underline{macroscopic}. We
call $\vec u_1,\mu_1,p_1$ the oscillating parts of $\vec u,\mu,p$
(they characterize the sound wave) and we assume that they are small
with respect to the averages $\vec u_0,\mu_0,p_0$ i.e. we have:
\begin{equation}\label{MaxVacFull1yiyiyninshtrgravortghhghgjkgghklhjgkghghjjkjhjkkggjkhjkhjjhhfhjhkjkhbbgjhzzzyyykkknnnNWBWHWPPNjjhjhjhjhplllhjggjl;kint}
|\vec u_1|\ll|\vec u_0|,\quad|\mu_1|\ll |\mu_0|,\quad |p_1|\ll
|p_0|.
\end{equation}
However, we assume that $\vec u_1,\mu_1,p_1$ are highly oscillate
and thus they changes spatially and temporary much faster than the
averages  $\vec u_0,\mu_0,p_0$ and the fields $\vec v, \vec E,\vec
B,\rho,\vec j$ i.e. we have:
\begin{multline}\label{MaxVacFull1yiyiyninshtrgravortghhghgjkgghklhjghjhnkkghghjjkjhjkkggjkhjkhjjhhfhjhkjkhbbgjhzzzyyykkknnnNWBWHWPPNjjhjhjhjhplllhjggjl;kint}
\frac{\left|d_{\vec x}\left(\rho\vec E+\frac{1}{c}\vec j\times\vec
B\right)\right|}{\left|\rho\vec E+\frac{1}{c}\vec j\times\vec
B\right|}+\frac{|d_{\vec x}\vec v|}{|\vec v|}+\frac{|d_{\vec
x}\mu_0|}{|\mu_0|}+\frac{|d_{\vec x}p_0|}{|p_0|}+\frac{|d_{\vec
x}\vec u_0|}{|\vec
u_0|}\ll \min{\left\{\frac{|d_{\vec x}\vec u_1|}{|\vec u_1|},\frac{|d_{\vec x}\mu_1|}{|\mu_1|},\frac{|d_{\vec x}p_1|}{|p_1|}\right\}},\\
\quad\quad \frac{\left|\partial_t\left(\rho\vec E+\frac{1}{c}\vec
j\times\vec B\right)\right|}{\left|\rho\vec E+\frac{1}{c}\vec
j\times\vec B\right|}+\frac{|\partial_t\vec v|}{|\vec
v|}+\frac{|\partial_t\mu_0|}{|\mu_0|}+\frac{|\partial_t
p_0|}{|p_0|}+\frac{|\partial_t\vec u_0|}{|\vec u_0|}\ll
\min{\left\{\frac{|\partial_t\vec u_1|}{|\vec
u_1|},\frac{|\partial_t\mu_1|}{|\mu_1|},\frac{|\partial_t
p_1|}{|p_1|}\right\}},\\ \frac{|\vec u_1|}{|\vec u_0|}\ll
\min{\left\{\frac{|d_{\vec x}\vec u_1|}{|d_{\vec x}\vec
u_0|},\frac{|d_{\vec x}\vec u_1|}{|d_{\vec x}\vec v|}\right\}},
\quad\quad \frac{|\mu_1|}{|\mu_0|}+\frac{|p_1|}{|p_0|}\ll
\min{\left\{\frac{|d_{\vec x}p_1|}{|d_{\vec x}p_0|},\frac{|d_{\vec
x}\mu_1|}{|d_{\vec x}\mu_0|}\right\}}\\ \quad\quad\text{and}
\quad\quad \frac{|\mu_1|}{|\mu_0|}\ll\frac{|d_{\vec x}\vec u_1||\vec
u_0||\mu_0|}{\left|\rho\vec E+\frac{1}{c}\vec j\times\vec B\right|}.
\end{multline}
Finally, we assume that the fields $\vec v, \vec E,\vec B,\rho,\vec
j$  and $(\Div_{\vec x}\vec q)$ change slowly with respect to the
oscillations of  $\vec u_1,\mu_1,p_1$ and thus we assume that $\vec
v, \vec E,\vec B,\rho,\vec j$  and $(\Div_{\vec x}\vec q)$ can be
replaced by their spatial and temporal averages. Note that
$\mu,p,\mu_0,p_0,\mu_1,p_1$ behave like \underline{proper scalar}
fields and $\vec u,\vec u_0$ behave like \underline{speed-like
vector} fields under the change of cartesian coordinate systems.
Thus, since $\vec u_1=\vec u-\vec u_0$, we deduce that $\vec u_1$
behaves like a \underline{proper vector} field  under the change of
cartesian coordinate systems. Finally, note that obviously the
averages of $\vec u_1,\mu_1,p_1$ vanish.

Then in subsection \ref{jihggghjgjgg} we deduce
\begin{equation}\label{MaxVacFull1ninshtrgravortghhghgjkgghklhjgkghghjjkjhjkkggjkhjkhjjhhfhjhkjhjjhkhbbgjlmkmmhzzzyyykkknnnNWBWHWPPNjjlllkkkkkkljkkjjkjkhjjhhgghggggint}
\frac{\partial}{\partial t}\left(\frac{\partial\mu_1}{\partial
t}+\vec u_0\cdot\nabla_{\vec x}\mu_1\right)+\Div_{\vec
x}\left\{\left(\frac{\partial\mu_1}{\partial t}+\vec
u_0\cdot\nabla_{\vec x}\mu_1\right) \vec u_0\right\}
\approx\Delta_{\vec x} p_1,
\end{equation}
and
\begin{equation}\label{MaxVacFull1ninshtrgravortghhghgjkgghklhjgkghghjjkjhjkkggjkhjkhjjhhfhjhkjhjjhkhbbgjlmkmmhzzzyyykkknnnNWBWHWPPNjjlllkkkkkkljkkjjkjkint}
\mu_0\left(\frac{\partial\vec u_1}{\partial t}-curl_{\vec
x}\left\{\vec u_0\times\vec u_1\right\}+(\Div_{\vec x}\vec
u_1)\,\vec u_0\right) \approx-\nabla_{\vec x} p_1.
\end{equation}
Note that, as before, it can be easily proved that
\er{MaxVacFull1ninshtrgravortghhghgjkgghklhjgkghghjjkjhjkkggjkhjkhjjhhfhjhkjhjjhkhbbgjlmkmmhzzzyyykkknnnNWBWHWPPNjjlllkkkkkkljkkjjkjkhjjhhgghggggint}
and
\er{MaxVacFull1ninshtrgravortghhghgjkgghklhjgkghghjjkjhjkkggjkhjkhjjhhfhjhkjhjjhkhbbgjlmkmmhzzzyyykkknnnNWBWHWPPNjjlllkkkkkkljkkjjkjkint}
are invariant under the change of inertial or non-inertial cartesian
coordinate system. Thus,
\er{MaxVacFull1ninshtrgravortghhghgjkgghklhjgkghghjjkjhjkkggjkhjkhjjhhfhjhkjhjjhkhbbgjlmkmmhzzzyyykkknnnNWBWHWPPNjjlllkkkkkkljkkjjkjkhjjhhgghggggint}
and
\er{MaxVacFull1ninshtrgravortghhghgjkgghklhjgkghghjjkjhjkkggjkhjkhjjhhfhjhkjhjjhkhbbgjlmkmmhzzzyyykkknnnNWBWHWPPNjjlllkkkkkkljkkjjkjkint}
are still valid if
\er{MaxVacFull1yiyiyninshtrgravortghhghgjkgghklhjgkghghjjkjhjkkggjkhjkhjjhhfhjhkjkhbbgjhzzzyyykkknnnNWBWHWPPNjjhjhjhjhplllhjggjl;kint}
and
\er{MaxVacFull1yiyiyninshtrgravortghhghgjkgghklhjghjhnkkghghjjkjhjkkggjkhjkhjjhhfhjhkjkhbbgjhzzzyyykkknnnNWBWHWPPNjjhjhjhjhplllhjggjl;kint}
hold only in some specific (possibly artificial) inertial or
non-inertial cartesian coordinate system.

Next we proceed in two cases:
\begin{itemize}
\item[{\bf (i)}]
In the case of the simple models of inviscid gas.
\item[{\bf (ii)}]
In the case of inviscid barotropic fluid.
\end{itemize}
In the case of inviscid gas, in subsection \ref{jihggghjgjgg}  we
prove that
\begin{equation}\label{MaxVacFull1ninshtrgravortghhghgjkgghklhjnnmnmmngkghghjjkjhjkkggjkhjkhjjhhfhjhkjkhbbgjhzzzyyykkknnnNWBWHWPPNjjhjhjhjhghgjkhkhlkkhhkhkhhhkhhjgjggjkkhhkjhjjhjhjhjhjjggint}
\mu_1\approx\left(\frac{p_0}{\mu_0}-\mu_0\frac{\partial
F}{\partial\mu}\left(\mu_0,p_0\right)\right)^{-1}\mu_0\frac{\partial
F}{\partial p}\left(\mu_0,p_0\right)\,p_1.
\end{equation}
Therefore, inserting
\er{MaxVacFull1ninshtrgravortghhghgjkgghklhjnnmnmmngkghghjjkjhjkkggjkhjkhjjhhfhjhkjkhbbgjhzzzyyykkknnnNWBWHWPPNjjhjhjhjhghgjkhkhlkkhhkhkhhhkhhjgjggjkkhhkjhjjhjhjhjhjjggint}
into
\er{MaxVacFull1ninshtrgravortghhghgjkgghklhjgkghghjjkjhjkkggjkhjkhjjhhfhjhkjhjjhkhbbgjlmkmmhzzzyyykkknnnNWBWHWPPNjjlllkkkkkkljkkjjkjkhjjhhgghggggint}
and using again
\er{MaxVacFull1yiyiyninshtrgravortghhghgjkgghklhjghjhnkkghghjjkjhjkkggjkhjkhjjhhfhjhkjkhbbgjhzzzyyykkknnnNWBWHWPPNjjhjhjhjhplllhjggjl;kint}
we finally obtain the wave equation for the oscillating part of the
pressure $p_1$ of the form:
\begin{equation}\label{MaxVacFull1ninshtrgravortghhghgjkgghklhjgkghghjjkjhjkkggjkhjkhjjhhfhjhkjhjjhkhbbgjlmkmmhzzzyyykkknnnNWBWHWPPNjjlllkkkkkkljkkjjkjkhjjhhgghgggghjjhgjgghint}
\frac{\partial}{\partial
t}\left\{\frac{1}{c^2_0}\left(\frac{\partial p_1}{\partial t}+\vec
u_0\cdot\nabla_{\vec x}p_1\right)\right\}+\Div_{\vec
x}\left\{\frac{1}{c^2_0}\left(\frac{\partial p_1}{\partial t}+\vec
u_0\cdot\nabla_{\vec x}p_1\right) \vec u_0\right\}
\approx\Delta_{\vec x} p_1,
\end{equation}
where we denote
\begin{equation}\label{MaxVacFull1ninshtrgravortghhghgjkgghklhjgkghghjjkjhjkkggjkhjkhjjhhfhjhkjhjjhkhbbgjlmkmmhzzzyyykkknnnNWBWHWPPNjjlllkkkkkkljkkjjkjkhjjhhgghgggghjjhgjgghhkhjhjhint}
c_0:=\sqrt{\frac{\left(\frac{p_0}{\mu_0}-\mu_0\frac{\partial
F}{\partial\mu}\left(\mu_0,p_0\right)\right)}{\mu_0\frac{\partial
F}{\partial p}\left(\mu_0,p_0\right)}}.
\end{equation}
Moreover, the oscillating parts of the density $\mu_1$ and the
velocity $\vec u_1$ can be found from
\er{MaxVacFull1ninshtrgravortghhghgjkgghklhjnnmnmmngkghghjjkjhjkkggjkhjkhjjhhfhjhkjkhbbgjhzzzyyykkknnnNWBWHWPPNjjhjhjhjhghgjkhkhlkkhhkhkhhhkhhjgjggjkkhhkjhjjhjhjhjhjjggint}
and
\er{MaxVacFull1ninshtrgravortghhghgjkgghklhjgkghghjjkjhjkkggjkhjkhjjhhfhjhkjhjjhkhbbgjlmkmmhzzzyyykkknnnNWBWHWPPNjjlllkkkkkkljkkjjkjkint}
respectively, i.e. we have
\begin{equation}\label{MaxVacFull1ninshtrgravortghhghgjkgghklhjnnmnmmngkghghjjkjhjkkggjkhjkhjjhhfhjhkjkhbbgjhzzzyyykkknnnNWBWHWPPNjjhjhjhjhghgjkhkhlkkhhkhkhhhkhhjgjggjkkhhkjhjjhjhjhjhjjggijjhjint}
\mu_1\approx\left(\frac{p_0}{\mu_0}-\mu_0\frac{\partial
F}{\partial\mu}\left(\mu_0,p_0\right)\right)^{-1}\mu_0\frac{\partial
F}{\partial p}\left(\mu_0,p_0\right)\,p_1,
\end{equation}
and
\begin{equation}\label{MaxVacFull1ninshtrgravortghhghgjkgghklhjgkghghjjkjhjkkggjkhjkhjjhhfhjhkjhjjhkhbbgjlmkmmhzzzyyykkknnnNWBWHWPPNjjlllkkkkkkljkkjjkjk9uu9uuiiuuiiuint}
\mu_0\left(\frac{\partial\vec u_1}{\partial t}-curl_{\vec
x}\left\{\vec u_0\times\vec u_1\right\}+(\Div_{\vec x}\vec
u_1)\,\vec u_0\right) \approx-\nabla_{\vec x} p_1.
\end{equation}
Note that, as before, it can be easily proved that
\er{MaxVacFull1ninshtrgravortghhghgjkgghklhjgkghghjjkjhjkkggjkhjkhjjhhfhjhkjhjjhkhbbgjlmkmmhzzzyyykkknnnNWBWHWPPNjjlllkkkkkkljkkjjkjkhjjhhgghgggghjjhgjgghint},
\er{MaxVacFull1ninshtrgravortghhghgjkgghklhjnnmnmmngkghghjjkjhjkkggjkhjkhjjhhfhjhkjkhbbgjhzzzyyykkknnnNWBWHWPPNjjhjhjhjhghgjkhkhlkkhhkhkhhhkhhjgjggjkkhhkjhjjhjhjhjhjjggijjhjint}
and
\er{MaxVacFull1ninshtrgravortghhghgjkgghklhjgkghghjjkjhjkkggjkhjkhjjhhfhjhkjhjjhkhbbgjlmkmmhzzzyyykkknnnNWBWHWPPNjjlllkkkkkkljkkjjkjk9uu9uuiiuuiiuint}
are invariant under the change of inertial or non-inertial cartesian
coordinate system. Thus, as before,
\er{MaxVacFull1ninshtrgravortghhghgjkgghklhjgkghghjjkjhjkkggjkhjkhjjhhfhjhkjhjjhkhbbgjlmkmmhzzzyyykkknnnNWBWHWPPNjjlllkkkkkkljkkjjkjkhjjhhgghgggghjjhgjgghint},
\er{MaxVacFull1ninshtrgravortghhghgjkgghklhjnnmnmmngkghghjjkjhjkkggjkhjkhjjhhfhjhkjkhbbgjhzzzyyykkknnnNWBWHWPPNjjhjhjhjhghgjkhkhlkkhhkhkhhhkhhjgjggjkkhhkjhjjhjhjhjhjjggijjhjint}
and
\er{MaxVacFull1ninshtrgravortghhghgjkgghklhjgkghghjjkjhjkkggjkhjkhjjhhfhjhkjhjjhkhbbgjlmkmmhzzzyyykkknnnNWBWHWPPNjjlllkkkkkkljkkjjkjk9uu9uuiiuuiiuint}
are  still valid if
\er{MaxVacFull1yiyiyninshtrgravortghhghgjkgghklhjgkghghjjkjhjkkggjkhjkhjjhhfhjhkjkhbbgjhzzzyyykkknnnNWBWHWPPNjjhjhjhjhplllhjggjl;kint}
and
\er{MaxVacFull1yiyiyninshtrgravortghhghgjkgghklhjghjhnkkghghjjkjhjkkggjkhjkhjjhhfhjhkjkhbbgjhzzzyyykkknnnNWBWHWPPNjjhjhjhjhplllhjggjl;kint}
hold only in some specific (possibly artificial) inertial or
non-inertial cartesian coordinate system.

In particular, in the case of in the case of the simplest ideal gas
where \er{noninchredPPNbjjhjjkhkhkljjljljklhjhjhjghhgghjhhjint} and
\er{noninchredPPNbjjhjjkhkhlkkkklhjgghghint} are valid, by
\er{MaxVacFull1nkkkinshtrgravortghhghgjkgghklhjgkghghjjkjhjkkggjkhjkhjjhhfhjhkjkhbbgjhzzzyyykkknnhjhhnNWBWHWPPNjjhjhjhjhplllhjggjll;;lijiihkhhkgjgjint}
we have
\begin{equation}\label{MaxVacFull1nkkkinshtrgravortghhghgjkgghklhjgkghghjjkjhjkkggjkhjkhjjhhfhjhkjkhbbgjhzzzyyykkknnhjhhnNWBWHWPPNjjhjhjhjhplllhjggjll;;lijiihkhhkgjgjhjhjhbccint}
\frac{E}{\mu}=F\left(\mu,p\right):=\gamma_0\,\frac{p}{\mu},
\end{equation}
and therefore, inserting
\er{MaxVacFull1nkkkinshtrgravortghhghgjkgghklhjgkghghjjkjhjkkggjkhjkhjjhhfhjhkjkhbbgjhzzzyyykkknnhjhhnNWBWHWPPNjjhjhjhjhplllhjggjll;;lijiihkhhkgjgjhjhjhbccint}
into
\er{MaxVacFull1ninshtrgravortghhghgjkgghklhjgkghghjjkjhjkkggjkhjkhjjhhfhjhkjhjjhkhbbgjlmkmmhzzzyyykkknnnNWBWHWPPNjjlllkkkkkkljkkjjkjkhjjhhgghgggghjjhgjgghhkhjhjhint}
gives:
\begin{equation}\label{MaxVacFull1ninshtrgravortghhghgjkgghklhjgkghghjjkjhjkkggjkhjkhjjhhfhjhkjhjjhkhbbgjlmkmmhzzzyyykkknnnNWBWHWPPNjjlllkkkkkkljkkjjkjkhjjhhgghgggghjjhgjgghhkhjhjhjjggjjgjgjgint}
c_0:=\sqrt{\frac{\left(1+\gamma_0\right)p_0}{\gamma_0\mu_0}}.
\end{equation}
Moreover, in this case
\er{MaxVacFull1ninshtrgravortghhghgjkgghklhjnnmnmmngkghghjjkjhjkkggjkhjkhjjhhfhjhkjkhbbgjhzzzyyykkknnnNWBWHWPPNjjhjhjhjhghgjkhkhlkkhhkhkhhhkhhjgjggjkkhhkjhjjhjhjhjhjjggint}
reeds as:
\begin{equation}\label{MaxVacFull1ninshtrgravortghhghgjkgghklhjnnmnmmngkghghjjkjhjkkggjkhjkhjjhhfhjhkjkhbbgjhzzzyyykkknnnNWBWHWPPNjjhjhjhjhghgjkhkhlkkhhkhkhhhkhhjgjggjkkhhkjhjjhjhjhjhjjgghyuyhint}
\mu_1\approx\frac{\gamma_0}{1+\gamma_0}\,\frac{\mu_0}{p_0}\,p_1.
\end{equation}

On the other hand, in the case barotropic fluid the pressure $p$ is
a function of the density $\mu$ \underline{only}, i.e.
\begin{equation}\label{MaxVacFull1ninshtrgravortghhghgjkgghklhjnnmnmmngkghghjjkjhjkkggjkhjkhjjhhfhjhkjkhbbgjhzzzyyykkknnnNWBWHWPPNjjhjhjhjhghgjkhkhlkkhhjggggggggggggggggggggggggggggggint}
p=\mathcal{R}(\mu).
\end{equation}
Thus, inserting
\er{MaxVacFull1ninshtrgravortghhghgjkgghklhjgkghghjjkjhjkkggjkhjkhjjhhfhjhkjkhbbgjhzzzyyykkknnnNWBWHWPPNjjhjhjhjhplllhjggjlint}
into
\er{MaxVacFull1ninshtrgravortghhghgjkgghklhjnnmnmmngkghghjjkjhjkkggjkhjkhjjhhfhjhkjkhbbgjhzzzyyykkknnnNWBWHWPPNjjhjhjhjhghgjkhkhlkkhhjggggggggggggggggggggggggggggggint}
and using
\er{MaxVacFull1yiyiyninshtrgravortghhghgjkgghklhjgkghghjjkjhjkkggjkhjkhjjhhfhjhkjkhbbgjhzzzyyykkknnnNWBWHWPPNjjhjhjhjhplllhjggjl;kint}
we deduce
\begin{equation}\label{MaxVacFull1ninshtrgravortghhghgjkgghklhjnnmnmmngkghghjjkjhjkkggjkhjkhjjhhfhjhkjkhbbgjhzzzyyykkknnnNWBWHWPPNjjhjhjhjhghgjkhkhlkkhhjggggggggggggggggggggggggggggggjggjgint}
p_1\approx\frac{d\mathcal{R}}{d\mu}(\mu_0)\,\mu_1.
\end{equation}
Therefore, inserting
\er{MaxVacFull1ninshtrgravortghhghgjkgghklhjnnmnmmngkghghjjkjhjkkggjkhjkhjjhhfhjhkjkhbbgjhzzzyyykkknnnNWBWHWPPNjjhjhjhjhghgjkhkhlkkhhjggggggggggggggggggggggggggggggjggjgint}
into
\er{MaxVacFull1ninshtrgravortghhghgjkgghklhjgkghghjjkjhjkkggjkhjkhjjhhfhjhkjhjjhkhbbgjlmkmmhzzzyyykkknnnNWBWHWPPNjjlllkkkkkkljkkjjkjkhjjhhgghggggint}
and using again
\er{MaxVacFull1yiyiyninshtrgravortghhghgjkgghklhjghjhnkkghghjjkjhjkkggjkhjkhjjhhfhjhkjkhbbgjhzzzyyykkknnnNWBWHWPPNjjhjhjhjhplllhjggjl;kint}
we finally obtain the analogous to
\er{MaxVacFull1ninshtrgravortghhghgjkgghklhjgkghghjjkjhjkkggjkhjkhjjhhfhjhkjhjjhkhbbgjlmkmmhzzzyyykkknnnNWBWHWPPNjjlllkkkkkkljkkjjkjkhjjhhgghgggghjjhgjgghint}
wave equation for the oscillating part of the pressure $p_1$ of the
form:
\begin{equation}\label{MaxVacFull1ninshtrgravortghhghgjkgghklhjgkghghjjkjhjkkggjkhjkhjjhhfhjhkjhjjhkhbbgjlmkmmhzzzyyykkknnnNWBWHWPPNjjlllkkkkkkljkkjjkjkhjjhhgghgggghjjhgjgghjkjkjkint}
\frac{\partial}{\partial
t}\left\{\frac{1}{c^2_0}\left(\frac{\partial p_1}{\partial t}+\vec
u_0\cdot\nabla_{\vec x}p_1\right)\right\}+\Div_{\vec
x}\left\{\frac{1}{c^2_0}\left(\frac{\partial p_1}{\partial t}+\vec
u_0\cdot\nabla_{\vec x}p_1\right) \vec u_0\right\}
\approx\Delta_{\vec x} p_1,
\end{equation}
where we denote
\begin{equation}\label{MaxVacFull1ninshtrgravortghhghgjkgghklhjgkghghjjkjhjkkggjkhjkhjjhhfhjhkjhjjhkhbbgjlmkmmhzzzyyykkknnnNWBWHWPPNjjlllkkkkkkljkkjjkjkhjjhhgghgggghjjhgjgghhkhjhjhhkhhhhgint}
c_0:=\sqrt{\frac{d\mathcal{R}}{d\mu}(\mu_0)}.
\end{equation}
Moreover, the oscillating parts of the density $\mu_1$ and the
velocity $\vec u_1$ can be found from
\er{MaxVacFull1ninshtrgravortghhghgjkgghklhjnnmnmmngkghghjjkjhjkkggjkhjkhjjhhfhjhkjkhbbgjhzzzyyykkknnnNWBWHWPPNjjhjhjhjhghgjkhkhlkkhhjggggggggggggggggggggggggggggggjggjgint}
and
\er{MaxVacFull1ninshtrgravortghhghgjkgghklhjgkghghjjkjhjkkggjkhjkhjjhhfhjhkjhjjhkhbbgjlmkmmhzzzyyykkknnnNWBWHWPPNjjlllkkkkkkljkkjjkjkint}
respectively, i.e. we have
\begin{equation}\label{hihihhhooint}
p_1\approx\frac{d\mathcal{R}}{d\mu}(\mu_0)\,\mu_1,
\end{equation}
and
\begin{equation}\label{MaxVacFull1ninshtrgravortghhghgjkgghklhjgkghghjjkjhjkkggjkhjkhjjhhfhjhkjhjjhkhbbgjlmkmmhzzzyyykkknnnNWBWHWPPNjjlllkkkkkkljkkjjkjk9uu9uuiiuuiiujkhjhjjojjint}
\mu_0\left(\frac{\partial\vec u_1}{\partial t}-curl_{\vec
x}\left\{\vec u_0\times\vec u_1\right\}+(\Div_{\vec x}\vec
u_1)\,\vec u_0\right) \approx-\nabla_{\vec x} p_1.
\end{equation}
Again note that
\er{MaxVacFull1ninshtrgravortghhghgjkgghklhjgkghghjjkjhjkkggjkhjkhjjhhfhjhkjhjjhkhbbgjlmkmmhzzzyyykkknnnNWBWHWPPNjjlllkkkkkkljkkjjkjkhjjhhgghgggghjjhgjgghjkjkjkint},
\er{hihihhhooint} and
\er{MaxVacFull1ninshtrgravortghhghgjkgghklhjgkghghjjkjhjkkggjkhjkhjjhhfhjhkjhjjhkhbbgjlmkmmhzzzyyykkknnnNWBWHWPPNjjlllkkkkkkljkkjjkjk9uu9uuiiuuiiujkhjhjjojjint}
are invariant under the change of inertial or non-inertial cartesian
coordinate system. Thus, as before,
\er{MaxVacFull1ninshtrgravortghhghgjkgghklhjgkghghjjkjhjkkggjkhjkhjjhhfhjhkjhjjhkhbbgjlmkmmhzzzyyykkknnnNWBWHWPPNjjlllkkkkkkljkkjjkjkhjjhhgghgggghjjhgjgghjkjkjkint},
\er{hihihhhooint} and
\er{MaxVacFull1ninshtrgravortghhghgjkgghklhjgkghghjjkjhjkkggjkhjkhjjhhfhjhkjhjjhkhbbgjlmkmmhzzzyyykkknnnNWBWHWPPNjjlllkkkkkkljkkjjkjk9uu9uuiiuuiiujkhjhjjojjint}
are  still valid if
\er{MaxVacFull1yiyiyninshtrgravortghhghgjkgghklhjgkghghjjkjhjkkggjkhjkhjjhhfhjhkjkhbbgjhzzzyyykkknnnNWBWHWPPNjjhjhjhjhplllhjggjl;kint}
and
\er{MaxVacFull1yiyiyninshtrgravortghhghgjkgghklhjghjhnkkghghjjkjhjkkggjkhjkhjjhhfhjhkjkhbbgjhzzzyyykkknnnNWBWHWPPNjjhjhjhjhplllhjggjl;kint}
hold only in some specific (possibly artificial) inertial or
non-inertial cartesian coordinate system.

\subsubsection{Propagation of waves in the moving elastic
body}\label{gjgfhffhfhint} Consistently with the general equation of
motion of a continuum medium
\er{MaxVacFull1ninshtrgravortghhghgjkgghklhjgkghghjjkjhjkkggjkhjkhjjhhfhjhkjkhbbgjhzzzyyykkknnnNWBWHWPPNjjint},
consider the motion of an elastic body:
\begin{multline}\label{MaxVacFull1ninshtrgravortghhghgjkgghklhjgkghghjjkjhjkkggjkhjkhjjhhfhjhkjkhbbgjhzzzyyykkknnnNWBWHWPPNjjjljkjkkhhhhhhhhhhhhint}
\mu\left(\frac{\partial\vec u}{\partial t}+d_{\vec x}\vec u\cdot
\vec u\right)=\\-\mu\vec u\times curl_{\vec x}\vec
v+\mu\left(\partial_{t}\vec v+\nabla_{\vec x}\frac{1}{2}|\vec
v|^2\right)+\rho\vec E+\frac{1}{c}\vec j\times\vec B+\Div_{\vec
x}\left\{\alpha\mathcal{E}(\vec
x,t)+\beta\left(tr\left\{\mathcal{E}(\vec
x,t)\right\}\right)I\right\}.
\end{multline}
where
\begin{equation}\label{hjgjgjgkllujkjjjkhhjhljjkjkljjljjhhhint}
\mathcal{E}(\vec x,t):=\frac{1}{2}\left(I-\left\{d_{\vec x}\vec
f(\vec x,t)\right\}^T\cdot d_{\vec x}\vec f(\vec x,t)\right),
\end{equation}
consistently with \er{hjgjgjgkllujkjjjkhhjhljjkjkljjljjint} and
\er{MaxVacFull1ninshtrgravortghhghgjkgghklhjgkghghjjkjhjkkggjkhjkhjjhhfhjhkhjjkhjgzzzyyykkknnnNWBWHWPPNjjhhjhjkphkhkhhklklklint},
with
\begin{equation}\label{hjgghhghhffhkkmkljjnnnnghgghghint}
\begin{cases}
\frac{\partial\vec f}{\partial t}\left(\vec x,t\right)+d_{\vec
x}\vec f\left(\vec x,t\right)\cdot\vec u\left(\vec x,t\right)=0
\\
\vec f(\vec x,t_0)=\vec x,
\end{cases}
\end{equation}
consistently with \er{hjgghhghhffhkkmkljjnnnnint}, where the initial
instant of time $t_0$ is chosen (possibly fictitiously) in such a
way that our elastic body is relaxed at time $t_0$. Moreover, by
\er{hjhjjttuttutukjhkhlkkjkhjljjzzint}
we have
\begin{equation}\label{hjhjjttuttutukjhkhlkkjkhjljjint}
\mu\left(\vec x,t\right)=\mu(\vec x,t_0)\left|det\{d_{\vec x}\vec
f(\vec x,t)\}\right|.
\end{equation}

Next we assume that
\begin{multline}\label{MaxVacFull1ninshtrgravortghhghgjkgghklhjgkghghjjkjhjkkggjkhjkhjjhhfhjhkjkhbbgjhzzzyyykkknnnNWBWHWPPNjjhjhjhjhplllhjggjlccint}
\vec u(\vec x,t)=\vec u_0(\vec x,t)+\vec u_1(\vec x,t),\quad
\mu(\vec x,t)=\mu_0(\vec x,t)+\mu_1(\vec x,t),\\ \vec f(\vec
x,t)=\vec f_0(\vec x,t)+\vec f_1(\vec x,t)\quad\text{and}\quad
\mathcal{E}(\vec x,t)=\mathcal{E}_{0}(\vec x,t)+\mathcal{E}_{1}(\vec
x,t)
\end{multline}
where $\vec u_0(\vec x,t),\mu_0(\vec x,t),\vec f_0(\vec
x,t),\mathcal{E}_{0}(\vec x,t)$ are the averages of $\vec u(\vec
x,t),\mu(\vec x,t),\vec f(\vec x,t),\mathcal{E}(\vec x,t)$ on small
spatial and temporal intervals, surrounding the point $(\vec x,t)$.
Although we assume these intervals of space and time to be very
small, we also assume them to be quite \underline{macroscopic}. We
call $\vec u_1,\mu_1,\vec f_1,\mathcal{E}_{1}$ the oscillating parts
of $\vec u,\mu,\vec f,\mathcal{E}$ and we assume that they are small
with respect to the averages $\vec u_0,\mu_0,\vec f_0,\mathcal{E}_0$
i.e. we have:
\begin{equation}\label{MaxVacFull1yiyiyninshtrgravortghhghgjkgghklhjgkghghjjkjhjkkggjkhjkhjjhhfhjhkjkhbbgjhzzzyyykkknnnNWBWHWPPNjjhjhjhjhplllhjggjl;kccint}
|\vec u_1|\ll|\vec u_0|,\quad|\mu_1|\ll |\mu_0|,\quad|\vec f_1|\ll
|\vec f_0|.
\end{equation}
However, we assume that $\vec u_1,\mu_1,\vec f_1$ are highly
oscillate and thus they changes spatially and temporary much faster
than the averages  $\vec u_0,\mu_0,\vec f_0$ and the fields $\vec v,
\vec E,\vec B,\rho,\vec j$, i.e. we have:
\begin{multline}\label{MaxVacFull1yiyiyninshtrgravortghhghgjkgghklhjghjhnkkghghjjkjhjkkggjkhjkhjjhhfhjhkjkhbbgjhzzzyyykkknnnNWBWHWPPNjjhjhjhjhplllhjggjl;kccint}
\frac{|d_{\vec x}\alpha|}{|\alpha|}+\frac{|d_{\vec
x}\beta|}{|\beta|}+\frac{\left|d_{\vec x}\left(\rho\vec
E+\frac{1}{c}\vec j\times\vec B\right)\right|}{\left|\rho\vec
E+\frac{1}{c}\vec j\times\vec B\right|}+\frac{|d_{\vec x}\vec
v|}{|\vec v|}+\frac{|d_{\vec x}\mu_0|}{|\mu_0|}+\frac{|d_{\vec
x}\vec f_0|}{|\vec f_0|} +\frac{|d_{\vec x}\vec u_0|}{|\vec u_0|}
\ll \min{\left\{\frac{|d_{\vec x}\vec u_1|}{|\vec u_1|},\frac{|d_{\vec x}\mu_1|}{|\mu_1|},\frac{|d_{\vec x}\vec f_1|}{|\vec f_1|}\right\}},\\
\quad\quad
\frac{\left|\partial_t\left(\rho\vec E+\frac{1}{c}\vec j\times\vec
B\right)\right|}{\left|\rho\vec E+\frac{1}{c}\vec j\times\vec
B\right|}+\frac{|\partial_t\vec v|}{|\vec
v|}+\frac{|\partial_t\mu_0|}{|\mu_0|}+\frac{|\partial_t \vec
f_0|}{|\vec f_0|}+\frac{|\partial_t\vec u_0|}{|\vec u_0|}\ll
\min{\left\{\frac{|\partial_t\vec u_1|}{|\vec
u_1|},\frac{|\partial_t\mu_1|}{|\mu_1|},\frac{|\partial_t \vec
f_1|}{|\vec f_1|}\right\}},\\ \frac{|d^2_{\vec x}\vec f_0|}{|d_{\vec
x}\vec f_0|}\ll\frac{|d^2_{\vec x}\vec f_1|}{|d_{\vec x}\vec f_1|}
\quad\quad \frac{|\vec u_1|}{|\vec u_0|}\ll
\min{\left\{\frac{|d_{\vec x}\vec u_1|}{|d_{\vec x}\vec
u_0|},\frac{|d_{\vec x}\vec u_1|}{|d_{\vec x}\vec v|}\right\}},
\quad\quad\text{and} \quad\quad
\frac{|\mu_1|}{|\mu_0|}\ll\frac{|d_{\vec x}\vec u_1||\vec
u_0||\mu_0|}{\left|\rho\vec E+\frac{1}{c}\vec j\times\vec B\right|}.
\end{multline}
Finally, we assume that the fields $\vec v, \vec E,\vec B,\rho,\vec
j$ change slowly with respect to the oscillations of  $\vec
u_1,\mu_1,\vec f_1$ and thus we assume that $\vec v, \vec E,\vec
B,\rho,\vec j$ can be replaced by their spatial and temporal
averages. Note that $\mu,\mu_0,\mu_1$ behave like \underline{proper
scalar} fields, $\mathcal{E},\mathcal{E}_0,\mathcal{E}_1$ behave
like \underline{proper matrix} fields and $\vec u,\vec u_0$ behave
like \underline{speed-like vector}  fields under the change of
cartesian coordinate systems. Thus, since $\vec u_1=\vec u-\vec
u_0$, we deduce that $\vec u_1$ behaves like a \underline{proper
vector} field  under the change of cartesian coordinate systems.
Furthermore, using \er{hjgjgjgkllujkjjjkhhjhljjkint} we deduce that
the vector field $\vec g_1\left(\vec x,t\right)$ defined by the
following:
\begin{equation}\label{jgjgghghghhkhkhkhint}
\vec g_1\left(\vec x,t\right):=\left(d_{\vec x}\vec f_0\left(\vec
x,t\right)\right)^{-1}\cdot\vec f_1\left(\vec x,t\right),
\end{equation}
behaves like a \underline{proper vector} field. Finally, note that
obviously the averages of $\vec u_1,\mu_1,\vec f_1,\mathcal{E}_1$
vanish.

Then in subsection \ref{gjgfhffhfh} we deduce
\begin{multline}\label{MaxVacFull1ninshtrgravortghhghgjkgghklhjgkghghjjkjhjkkggjkhjkhjjhhfhjhkjkhbbgjhzzzyyykkknnnNWBWHWPPNjjjljkjkkhhhhhhhhhhhhfnfffhuyuyuyutjhjhjhggjjfgfhfhint}
\mu_0\left(\frac{\partial}{\partial t}\left(\frac{\partial\vec
g_1}{\partial t}+d_{\vec x}\vec g_1\cdot\vec u_0\right)+d_{\vec
x}\left(\frac{\partial\vec g_1}{\partial t}+d_{\vec x}\vec g_1\cdot\vec u_0\right)\cdot \vec u_0\right)\approx\\
-\alpha\,\Div_{\vec x}\left\{\mathcal{E}_0\cdot d_{\vec x}\vec
g_1\right\}-\alpha\,\Div_{\vec x}\left\{\left\{ d_{\vec x}\vec
g_1\right\}^T\cdot\mathcal{E}_0\right\}-\beta\,\nabla_{\vec
x}\left(tr\left(\mathcal{E}_0\cdot d_{\vec x}\vec
g_1\right)\right)-\beta\,\nabla_{\vec x}\left(tr\left(\left\{
d_{\vec x}\vec
g_1\right\}^T\cdot\mathcal{E}_0\right)\right)\\+\frac{\alpha}{2}\,\Delta_{\vec
x}\vec g_1+\left(\frac{\alpha}{2}+\beta\right)\nabla_{\vec
x}\left(\Div_{\vec x}\vec g_1\right),
\end{multline}
and
\begin{equation}\label{hjgghhghhffhkkmkljjnnnnghgghghjhyjyyjhhjhgjgjjhjhjhjhggjghjhnhmnmmbjoujhkhhhint}
\frac{\partial\vec g_1}{\partial t}-curl_{\vec x}\left\{\vec
u_0\times\vec g_1\right\}+\left(\Div_{\vec x}\vec g_1\right)\vec
u_0\approx-\vec u_1,
\end{equation}
where $\vec g_1(\vec x,t)$ is given by \er{jgjgghghghhkhkhkhint}.
Furthermore, we also derive
\begin{equation}\label{hjgjgjgkllujkjjjkhhjhljjkjkljjljjhhhjkhhuggint}
\mathcal{E}_0(\vec x,t)\approx\frac{1}{2}\left(I-\left\{d_{\vec
x}\vec f_0(\vec x,t)\right\}^T\cdot d_{\vec x}\vec f_0(\vec
x,t)\right),
\end{equation}
and
\begin{multline}\label{hjgjgjgkllujkjjjkhhjhljjkjkljjljjhkhjhhffgjjjljjlmbint}
\mathcal{E}_{1}(\vec x,t)\approx \mathcal{E}_0(\vec x,t)\cdot
d_{\vec x}\vec g_1(\vec x,t)+\left\{ d_{\vec x}\vec g_1(\vec
x,t)\right\}^T\cdot\mathcal{E}_0(\vec x,t)-\frac{1}{2}\left(d_{\vec
x}\vec g_1(\vec x,t)+\left\{ d_{\vec x}\vec g_1(\vec
x,t)\right\}^T\right).
\end{multline}
Moreover, we can derive the approximation of
\er{MaxVacFull1ninshtrgravortghhghgjkgghklhjgkghghjjkjhjkkggjkhjkhjjhhfhjhkjkhbbgjhzzzyyykkknnnNWBWHWPPNjjjljkjkkhhhhhhhhhhhhfnfffhuyuyuyutjhjhjhggjjfgfhfhint}
as
\begin{multline}\label{MaxVacFull1ninshtrgravortghhghgjkgghklhjgkghghjjkjhjkkggjkhjkhjjhhfhjhkjkhbbgjhzzzyyykkknnnNWBWHWPPNjjjljkjkkhhhhhhhhhhhhfnfffhuyuyuyutjhjhjhggjjfgfhfhkhhjhjhint}
\mu_0\frac{\partial}{\partial t}\left(\frac{\partial\vec
g_1}{\partial t}-curl_{\vec x}\left\{\vec u_0\times\vec
g_1\right\}+\left(\Div_{\vec x}\vec g_1\right)\vec
u_0\right)\\-\mu_0\,curl_{\vec x}\left\{\vec
u_0\times\left(\frac{\partial\vec g_1}{\partial t}-curl_{\vec
x}\left\{\vec u_0\times\vec g_1\right\}+\left(\Div_{\vec x}\vec
g_1\right)\vec u_0\right)\right\}\\+\left(\Div_{\vec
x}\left\{\frac{\partial\vec g_1}{\partial t}-curl_{\vec
x}\left\{\vec u_0\times\vec g_1\right\}+\left(\Div_{\vec x}\vec
g_1\right)\vec u_0\right\}\right)\vec u_0
\approx\frac{\alpha}{2}\,\Delta_{\vec x}\vec
g_1+\left(\frac{\alpha}{2}+\beta\right)\nabla_{\vec
x}\left(\Div_{\vec x}\vec g_1\right)\\
-\alpha\,\Div_{\vec x}\left\{\mathcal{E}_0\cdot d_{\vec x}\vec
g_1\right\}-\alpha\,\Div_{\vec x}\left\{\left\{ d_{\vec x}\vec
g_1\right\}^T\cdot\mathcal{E}_0\right\}-\beta\,\nabla_{\vec
x}\left(tr\left(\mathcal{E}_0\cdot d_{\vec x}\vec
g_1\right)\right)-\beta\,\nabla_{\vec x}\left(tr\left(\left\{
d_{\vec x}\vec g_1\right\}^T\cdot\mathcal{E}_0\right)\right).
\end{multline}
In, particular, if our elastic body is \underline{nearly rigid} we
have
\begin{equation}\label{hjgjgjgkllujkjjjkhhjhljjkjkljjljjhhhjkhhugghjjhjhgjint}
\mathcal{E}_0\approx 0,
\end{equation}
and we simplify
\er{MaxVacFull1ninshtrgravortghhghgjkgghklhjgkghghjjkjhjkkggjkhjkhjjhhfhjhkjkhbbgjhzzzyyykkknnnNWBWHWPPNjjjljkjkkhhhhhhhhhhhhfnfffhuyuyuyutjhjhjhggjjfgfhfhint}
as:
\begin{multline}\label{MaxVacFull1ninshtrgravortghhghgjkgghklhjgkghghjjkjhjkkggjkhjkhjjhhfhjhkjkhbbgjhzzzyyykkknnnNWBWHWPPNjjjljkjkkhhhhhhhhhhhhfnfffhuyuyuyutjhjhjhggjjfgfhfhjhint}
\mu_0\left(\frac{\partial}{\partial t}\left(\frac{\partial\vec
g_1}{\partial t}+d_{\vec x}\vec g_1\cdot\vec u_0\right)+d_{\vec
x}\left(\frac{\partial\vec g_1}{\partial t}+d_{\vec x}\vec
g_1\cdot\vec u_0\right)\cdot \vec u_0\right)
\approx\frac{\alpha}{2}\,\Delta_{\vec x}\vec
g_1+\left(\frac{\alpha}{2}+\beta\right)\nabla_{\vec
x}\left(\Div_{\vec x}\vec g_1\right),
\end{multline}
and \er{hjgjgjgkllujkjjjkhhjhljjkjkljjljjhkhjhhffgjjjljjlmbint} as:
\begin{equation}\label{hjgjgjgkllujkjjjkhhjhljjkjkljjljjhkhjhhffgjjjljjlmbhkhkhkhkint}
\mathcal{E}_{1}\approx-\frac{1}{2}\left(d_{\vec x}\vec g_1+\left\{
d_{\vec x}\vec g_1\right\}^T\right).
\end{equation}
Moreover, we rewrite
\er{MaxVacFull1ninshtrgravortghhghgjkgghklhjgkghghjjkjhjkkggjkhjkhjjhhfhjhkjkhbbgjhzzzyyykkknnnNWBWHWPPNjjjljkjkkhhhhhhhhhhhhfnfffhuyuyuyutjhjhjhggjjfgfhfhkhhjhjhint}
as:
\begin{multline}\label{MaxVacFull1ninshtrgravortghhghgjkgghklhjgkghghjjkjhjkkggjkhjkhjjhhfhjhkjkhbbgjhzzzyyykkknnnNWBWHWPPNjjjljkjkkhhhhhhhhhhhhfnfffhuyuyuyutjhjhjhggjjfgfhfhjhhgjgjint}
\mu_0\frac{\partial}{\partial t}\left(\frac{\partial\vec
g_1}{\partial t}-curl_{\vec x}\left\{\vec u_0\times\vec
g_1\right\}+\left(\Div_{\vec x}\vec g_1\right)\vec
u_0\right)\\-\mu_0\,curl_{\vec x}\left\{\vec
u_0\times\left(\frac{\partial\vec g_1}{\partial t}-curl_{\vec
x}\left\{\vec u_0\times\vec g_1\right\}+\left(\Div_{\vec x}\vec
g_1\right)\vec u_0\right)\right\}\\+\mu_0\left(\Div_{\vec
x}\left\{\frac{\partial\vec g_1}{\partial t}-curl_{\vec
x}\left\{\vec u_0\times\vec g_1\right\}+\left(\Div_{\vec x}\vec
g_1\right)\vec u_0\right\}\right)\vec u_0
\approx\frac{\alpha}{2}\,\Delta_{\vec x}\vec
g_1+\left(\frac{\alpha}{2}+\beta\right)\nabla_{\vec
x}\left(\Div_{\vec x}\vec g_1\right).
\end{multline}
Next, note that, since $\mu,\mu_0,\mu_1$ are \underline{proper
scalar} fields, $\mathcal{E},\mathcal{E}_0,\mathcal{E}_1$ are
\underline{proper matrix} fields, $\vec u,\vec u_0$ are
\underline{speed-like vector} fields, $\vec u_1$ is a
\underline{proper vector} field and the vector field $\vec
g_1\left(\vec x,t\right)$ defined by \er{jgjgghghghhkhkhkhint} is a
\underline{proper vector} field, as before, it can be easily proved
that
\er{MaxVacFull1ninshtrgravortghhghgjkgghklhjgkghghjjkjhjkkggjkhjkhjjhhfhjhkjkhbbgjhzzzyyykkknnnNWBWHWPPNjjjljkjkkhhhhhhhhhhhhfnfffhuyuyuyutjhjhjhggjjfgfhfhkhhjhjhint},
\er{hjgghhghhffhkkmkljjnnnnghgghghjhyjyyjhhjhgjgjjhjhjhjhggjghjhnhmnmmbjoujhkhhhint},
\er{hjgjgjgkllujkjjjkhhjhljjkjkljjljjhhhjkhhuggint},
\er{hjgjgjgkllujkjjjkhhjhljjkjkljjljjhkhjhhffgjjjljjlmbint},
\er{MaxVacFull1ninshtrgravortghhghgjkgghklhjgkghghjjkjhjkkggjkhjkhjjhhfhjhkjkhbbgjhzzzyyykkknnnNWBWHWPPNjjjljkjkkhhhhhhhhhhhhfnfffhuyuyuyutjhjhjhggjjfgfhfhjhhgjgjint}
\er{hjgjgjgkllujkjjjkhhjhljjkjkljjljjhhhjkhhugghjjhjhgjint} and
\er{hjgjgjgkllujkjjjkhhjhljjkjkljjljjhkhjhhffgjjjljjlmbhkhkhkhkint}
are invariant under the change of inertial or non-inertial cartesian
coordinate system. Thus,
\er{MaxVacFull1ninshtrgravortghhghgjkgghklhjgkghghjjkjhjkkggjkhjkhjjhhfhjhkjkhbbgjhzzzyyykkknnnNWBWHWPPNjjjljkjkkhhhhhhhhhhhhfnfffhuyuyuyutjhjhjhggjjfgfhfhkhhjhjhint},
\er{hjgghhghhffhkkmkljjnnnnghgghghjhyjyyjhhjhgjgjjhjhjhjhggjghjhnhmnmmbjoujhkhhhint},
\er{hjgjgjgkllujkjjjkhhjhljjkjkljjljjhhhjkhhuggint},
\er{hjgjgjgkllujkjjjkhhjhljjkjkljjljjhkhjhhffgjjjljjlmbint} and in
the case of \underline{nearly rigid} elastic body also
\er{MaxVacFull1ninshtrgravortghhghgjkgghklhjgkghghjjkjhjkkggjkhjkhjjhhfhjhkjkhbbgjhzzzyyykkknnnNWBWHWPPNjjjljkjkkhhhhhhhhhhhhfnfffhuyuyuyutjhjhjhggjjfgfhfhjhhgjgjint}
\er{hjgjgjgkllujkjjjkhhjhljjkjkljjljjhhhjkhhugghjjhjhgjint} and
\er{hjgjgjgkllujkjjjkhhjhljjkjkljjljjhkhjhhffgjjjljjlmbhkhkhkhkint}
are still valid if
\er{MaxVacFull1yiyiyninshtrgravortghhghgjkgghklhjgkghghjjkjhjkkggjkhjkhjjhhfhjhkjkhbbgjhzzzyyykkknnnNWBWHWPPNjjhjhjhjhplllhjggjl;kccint}
and
\er{MaxVacFull1yiyiyninshtrgravortghhghgjkgghklhjghjhnkkghghjjkjhjkkggjkhjkhjjhhfhjhkjkhbbgjhzzzyyykkknnnNWBWHWPPNjjhjhjhjhplllhjggjl;kccint}
hold only in some specific (possibly artificial) inertial or
non-inertial cartesian coordinate system.

Next, taking the divergence of both sides of
\er{MaxVacFull1ninshtrgravortghhghgjkgghklhjgkghghjjkjhjkkggjkhjkhjjhhfhjhkjkhbbgjhzzzyyykkknnnNWBWHWPPNjjjljkjkkhhhhhhhhhhhhfnfffhuyuyuyutjhjhjhggjjfgfhfhjhhgjgjint}
and using
\er{MaxVacFull1yiyiyninshtrgravortghhghgjkgghklhjghjhnkkghghjjkjhjkkggjkhjkhjjhhfhjhkjkhbbgjhzzzyyykkknnnNWBWHWPPNjjhjhjhjhplllhjggjl;kccint}
gives:
\begin{equation}\label{MaxVacFull1ninshtrgravortghhghgjkgghklhjgkghghjjkjhjkkggjkhjkhjjhhfhjhkjkhbbgjhzzzyyykkknnnNWBWHWPPNjjjljkjkkhhhhhhhhhhhhfnfffhuyuyuyutjhjhjhggjjfgfhfhjhhgjgjhkjhjhkhkhint}
\frac{\partial}{\partial
t}\left\{\frac{1}{c^2_{01}}\left(\frac{\partial h_1}{\partial
t}+\vec u_0\cdot\nabla_{\vec x} h_1\right)\right\}+\Div_{\vec
x}\left\{\frac{1}{c^2_{01}}\left(\frac{\partial h_1}{\partial
t}+\vec u_0\cdot\nabla_{\vec x} h_1\right)\vec u_0\right\}
\approx\Delta_{\vec x} h_1,
\end{equation}
where
\begin{equation}\label{hjgjgjgkllujkjjjkhhjhljjkjkljjljjhhhjkhhugghjjhjhgjkhhhkhint}
h_1(\vec x,t):=\Div_{\vec x}\vec g_1(\vec
x,t)\quad\quad\text{and}\quad\quad
c_{01}:=\sqrt{\frac{\left(\alpha+\beta\right)}{\mu_0}}.
\end{equation}
On the other hand taking the curl of both sides of
\er{MaxVacFull1ninshtrgravortghhghgjkgghklhjgkghghjjkjhjkkggjkhjkhjjhhfhjhkjkhbbgjhzzzyyykkknnnNWBWHWPPNjjjljkjkkhhhhhhhhhhhhfnfffhuyuyuyutjhjhjhggjjfgfhfhjhint}
and using
\er{MaxVacFull1yiyiyninshtrgravortghhghgjkgghklhjghjhnkkghghjjkjhjkkggjkhjkhjjhhfhjhkjkhbbgjhzzzyyykkknnnNWBWHWPPNjjhjhjhjhplllhjggjl;kccint}
gives:
\begin{equation}\label{MaxVacFull1ninshtrgravortghhghgjkgghklhjgkghghjjkjhjkkggjkhjkhjjhhfhjhkjkhbbgjhzzzyyykkknnnNWBWHWPPNjjjljkjkkhhhhhhhhhhhhfnfffhuyuyuyutjhjhjhggjjfgfhfhjhgjgjkhkhgnnjkhint}
\frac{\partial}{\partial
t}\left\{\frac{1}{c^2_{02}}\left(\frac{\partial\vec h_2}{\partial
t}+d_{\vec x}\vec h_2\cdot\vec u_0\right)\right\}+\Div_{\vec
x}\left\{\frac{1}{c^2_{02}}\left(\frac{\partial\vec h_2}{\partial
t}+d_{\vec x}\vec h_2\cdot\vec u_0\right)\otimes\vec u_0\right\}
\approx\Delta_{\vec x}\vec h_2,
\end{equation}
where
\begin{equation}\label{hjgjgjgkllujkjjjkhhjhljjkjkljjljjhhhjkhhugghjjhjhgjkhhhkhhjjghjkhjhint}
\vec h_2(\vec x,t):=curl_{\vec x}\,\vec g_1(\vec
x,t)\quad\quad\text{and}\quad\quad
c_{02}:=\sqrt{\frac{\alpha}{2\mu_0}}.
\end{equation}
Again, using
\er{MaxVacFull1yiyiyninshtrgravortghhghgjkgghklhjghjhnkkghghjjkjhjkkggjkhjkhjjhhfhjhkjkhbbgjhzzzyyykkknnnNWBWHWPPNjjhjhjhjhplllhjggjl;kccint}
we can rewrite
\er{MaxVacFull1ninshtrgravortghhghgjkgghklhjgkghghjjkjhjkkggjkhjkhjjhhfhjhkjkhbbgjhzzzyyykkknnnNWBWHWPPNjjjljkjkkhhhhhhhhhhhhfnfffhuyuyuyutjhjhjhggjjfgfhfhjhgjgjkhkhgnnjkhint}
as:
\begin{multline}\label{MaxVacFullPPNnnnffffffyuughjhjhjhhjjkjhkkjhujgGGggffint}
\frac{1}{c^2_{02}}\frac{\partial}{\partial
t}\left(\frac{\partial\vec h_2}{\partial t}-\vec { u}_0\times
curl_{\vec x}\vec h_2+\nabla_{\vec x}\left(\vec h_2\cdot\vec {
u}_0\right)\right)\\-\frac{1}{c^2_{02}}curl_{\vec x} \left\{\vec
{u}_0\times
\left(\frac{\partial\vec h_2}{\partial t}-\vec {u}_0\times
curl_{\vec x}\vec h_2+\nabla_{\vec x}\left(\vec h_2\cdot\vec {
u}_0\right)\right)\right\}\\+\frac{1}{c^2_{02}}\left(\Div_{\vec
x}\left(\frac{\partial\vec h_2}{\partial t}-\vec {u}_0\times
curl_{\vec x}\vec h_2+\nabla_{\vec x}\left(\vec h_2\cdot\vec {
u}_0\right)\right)\right)\vec {u}_0\approx \Delta_{\vec x}\vec h_2.
\end{multline}
Next we deduce that equations
\er{MaxVacFull1ninshtrgravortghhghgjkgghklhjgkghghjjkjhjkkggjkhjkhjjhhfhjhkjkhbbgjhzzzyyykkknnnNWBWHWPPNjjjljkjkkhhhhhhhhhhhhfnfffhuyuyuyutjhjhjhggjjfgfhfhjhhgjgjhkjhjhkhkhint}
and \er{MaxVacFullPPNnnnffffffyuughjhjhjhhjjkjhkkjhujgGGggffint} are
invariant under the change of inertial or non-inertial cartesian
coordinate system, provided that $\vec h_2$ is a proper vector field
and $h_1$ is a proper scalar field. Thus, in the case of
\underline{nearly rigid} elastic body
\er{MaxVacFull1ninshtrgravortghhghgjkgghklhjgkghghjjkjhjkkggjkhjkhjjhhfhjhkjkhbbgjhzzzyyykkknnnNWBWHWPPNjjjljkjkkhhhhhhhhhhhhfnfffhuyuyuyutjhjhjhggjjfgfhfhjhhgjgjhkjhjhkhkhint}
and \er{MaxVacFullPPNnnnffffffyuughjhjhjhhjjkjhkkjhujgGGggffint} are
still valid if
\er{MaxVacFull1yiyiyninshtrgravortghhghgjkgghklhjgkghghjjkjhjkkggjkhjkhjjhhfhjhkjkhbbgjhzzzyyykkknnnNWBWHWPPNjjhjhjhjhplllhjggjl;kccint}
and
\er{MaxVacFull1yiyiyninshtrgravortghhghgjkgghklhjghjhnkkghghjjkjhjkkggjkhjkhjjhhfhjhkjkhbbgjhzzzyyykkknnnNWBWHWPPNjjhjhjhjhplllhjggjl;kccint}
hold only in some specific (possibly artificial) inertial or
non-inertial cartesian coordinate system. Finally note that,
although the equation
\er{MaxVacFull1ninshtrgravortghhghgjkgghklhjgkghghjjkjhjkkggjkhjkhjjhhfhjhkjkhbbgjhzzzyyykkknnnNWBWHWPPNjjjljkjkkhhhhhhhhhhhhfnfffhuyuyuyutjhjhjhggjjfgfhfhjhgjgjkhkhgnnjkhint}
is invariant under the Galilean Transformation, it is not invariant
under the more general change of non-inertial cartesian coordinate
system. However,
\er{MaxVacFull1ninshtrgravortghhghgjkgghklhjgkghghjjkjhjkkggjkhjkhjjhhfhjhkjkhbbgjhzzzyyykkknnnNWBWHWPPNjjjljkjkkhhhhhhhhhhhhfnfffhuyuyuyutjhjhjhggjjfgfhfhjhgjgjkhkhgnnjkhint}
is more convenient than
\er{MaxVacFullPPNnnnffffffyuughjhjhjhhjjkjhkkjhujgGGggffint}, since
every of the three scalar components of the vector field $\vec h_2$
in
\er{MaxVacFull1ninshtrgravortghhghgjkgghklhjgkghghjjkjhjkkggjkhjkhjjhhfhjhkjkhbbgjhzzzyyykkknnnNWBWHWPPNjjjljkjkkhhhhhhhhhhhhfnfffhuyuyuyutjhjhjhggjjfgfhfhjhgjgjkhkhgnnjkhint}
satisfies three decoupled wave equations of the same type. On the
other hand, if we consider some three proper vector fields $\vec
e_1:=\vec e_1(\vec x,t)$, $\vec e_2:=\vec e_2(\vec x,t)$, and $\vec
e_3:=\vec e_3(\vec x,t)$, which are mutually orthogonal to each
other and satisfy the following approximation analogous to
\er{MaxVacFull1yiyiyninshtrgravortghhghgjkgghklhjghjhnkkghghjjkjhjkkggjkhjkhjjhhfhjhkjkhbbgjhzzzyyykkknnnNWBWHWPPNjjhjhjhjhplllhjggjl;kccint}:
\begin{equation}\label{MaxVacFullPPNmmmffffffhhtygghGGGGyuhggghghhjhjhffint}
\frac{|d_{\vec x}\vec {e}_1|+c_{02}|\partial_t\vec e_1|}{|\vec
{e}_1|}+\frac{|d_{\vec x}\vec {e}_2|+c_{02}|\partial_t\vec
e_2|}{|\vec {e}_2|}+\frac{|d_{\vec x}\vec
{e}_3|+c_{02}|\partial_t\vec e_3|}{|\vec {e}_3|}\ll\frac{|d_{\vec
x}\vec h_2|+c_{02}|\partial_t\vec h_2|}{|\vec h_2|}.
\end{equation}
i.e. the field $\vec {e}_k$ vary in space and time much weaker than
$\vec h_2$, then we may write the alternative to
\er{MaxVacFull1ninshtrgravortghhghgjkgghklhjgkghghjjkjhjkkggjkhjkhjjhhfhjhkjkhbbgjhzzzyyykkknnnNWBWHWPPNjjjljkjkkhhhhhhhhhhhhfnfffhuyuyuyutjhjhjhggjjfgfhfhjhgjgjkhkhgnnjkhint}
approximate equations in the form of three decoupled scalar wave
equations of the same type:
\begin{multline}\label{MaxVacFullPPNmmmffffffiuiuhjuGGggFGjjkkjjffint}
\frac{\partial}{\partial
t}\left\{\frac{1}{c^2_{02}}\left(\frac{\partial}{\partial
t}\left(\vec e_k\cdot\vec h_2\right)+\vec {u}_0\cdot\nabla_{\vec
x}\left(\vec e_k\cdot\vec h_2\right)\right)\right\}+\Div_{\vec x}
\left\{\frac{1}{c^2_{02}}\left(\frac{\partial}{\partial t}\left(\vec
e_k\cdot\vec h_2\right)+\vec {u}_0\cdot\nabla_{\vec x}\left(\vec
e_k\cdot\vec h_2\right)\right)\vec {u}_0\right\}\\ \approx
\Delta_{\vec x}\left(\vec e_k\cdot\vec h_2\right)\quad\quad\forall
k=1,2,3.
\end{multline}
Then, clearly, the new alternative approximate equations
\er{MaxVacFullPPNmmmffffffiuiuhjuGGggFGjjkkjjffint} together with
\er{MaxVacFull1ninshtrgravortghhghgjkgghklhjgkghghjjkjhjkkggjkhjkhjjhhfhjhkjkhbbgjhzzzyyykkknnnNWBWHWPPNjjjljkjkkhhhhhhhhhhhhfnfffhuyuyuyutjhjhjhggjjfgfhfhjhhgjgjhkjhjhkhkhint}
are indeed invariant under the more general change of non-inertial
cartesian coordinate system.

As a final corollary of the above in the case of nearly rigid body
we have \er{MaxVacFullPPNmmmffffffiuiuhjuGGggFGjjkkjjffint} and
\er{MaxVacFull1ninshtrgravortghhghgjkgghklhjgkghghjjkjhjkkggjkhjkhjjhhfhjhkjkhbbgjhzzzyyykkknnnNWBWHWPPNjjjljkjkkhhhhhhhhhhhhfnfffhuyuyuyutjhjhjhggjjfgfhfhjhhgjgjhkjhjhkhkhint},
i.e.:
\begin{equation}\label{MaxVacFull1ninshtrgravortghhghgjkgghklhjgkghghjjkjhjkkggjkhjkhjjhhfhjhkjkhbbgjhzzzyyykkknnnNWBWHWPPNjjjljkjkkhhhhhhhhhhhhfnfffhuyuyuyutjhjhjhggjjfgfhfhjhhgjgjhkjhjhkhkhjjint}
\frac{\partial}{\partial
t}\left\{\frac{1}{c^2_{01}}\left(\frac{\partial h_1}{\partial
t}+\vec u_0\cdot\nabla_{\vec x} h_1\right)\right\}+\Div_{\vec
x}\left\{\frac{1}{c^2_{01}}\left(\frac{\partial h_1}{\partial
t}+\vec u_0\cdot\nabla_{\vec x} h_1\right)\vec u_0\right\}
\approx\Delta_{\vec x} h_1,
\end{equation}
and
\begin{multline}\label{MaxVacFullPPNmmmffffffiuiuhjuGGggFGjjkkjjffjjint}
\frac{\partial}{\partial
t}\left\{\frac{1}{c^2_{02}}\left(\frac{\partial}{\partial
t}\left(\vec e_k\cdot\vec h_2\right)+\vec {u}_0\cdot\nabla_{\vec
x}\left(\vec e_k\cdot\vec h_2\right)\right)\right\}+\Div_{\vec x}
\left\{\frac{1}{c^2_{02}}\left(\frac{\partial}{\partial t}\left(\vec
e_k\cdot\vec h_2\right)+\vec {u}_0\cdot\nabla_{\vec x}\left(\vec
e_k\cdot\vec h_2\right)\right)\vec {u}_0\right\}\\ \approx
\Delta_{\vec x}\left(\vec e_k\cdot\vec h_2\right)\quad\quad\forall
k=1,2,3,
\end{multline}
where
\begin{equation}\label{hjgjgjgkllujkjjjkhhjhljjkjkljjljjhhhjkhhugghjjhjhgjkhhhkhjkhhkhkkhint}
h_1(\vec x,t):=\Div_{\vec x}\vec g_1(\vec
x,t)\quad\quad\text{and}\quad\quad
c_{01}:=\sqrt{\frac{\left(\alpha+\beta\right)}{\mu_0}}.
\end{equation}
and
\begin{equation}\label{hjgjgjgkllujkjjjkhhjhljjkjkljjljjhhhjkhhugghjjhjhgjkhhhkhhjjghjkhjhjkjkjkkint}
\vec h_2(\vec x,t):=curl_{\vec x}\,\vec g_1(\vec
x,t)\quad\quad\text{and}\quad\quad
c_{02}:=\sqrt{\frac{\alpha}{2\mu_0}}.
\end{equation}
Moreover, we have
\er{hjgjgjgkllujkjjjkhhjhljjkjkljjljjhkhjhhffgjjjljjlmbhkhkhkhkint},
i.e.:
\begin{equation}\label{hjgjgjgkllujkjjjkhhjhljjkjkljjljjhkhjhhffgjjjljjlmbhkhkhkhkjkkhhhint}
\mathcal{E}_{1}\approx-\frac{1}{2}\left(d_{\vec x}\vec g_1+\left\{
d_{\vec x}\vec g_1\right\}^T\right),
\end{equation}
and
\er{hjgghhghhffhkkmkljjnnnnghgghghjhyjyyjhhjhgjgjjhjhjhjhggjghjhnhmnmmbjoujhkhhhint},
i.e.:
\begin{equation}\label{hjgghhghhffhkkmkljjnnnnghgghghjhyjyyjhhjhgjgjjhjhjhjhggjghjhnhmnmmbjoujhkggghghint}
\frac{\partial\vec g_1}{\partial t}-curl_{\vec x}\left\{\vec
u_0\times\vec g_1\right\}+\left(\Div_{\vec x}\vec g_1\right)\vec
u_0\approx-\vec u_1,
\end{equation}
and the above six equations are invariant under the change of
inertial or non-inertial cartesian coordinate system.

As we can easily see, the divergence of $\vec g_1$ and any of the
scalar components of the curl of $\vec g_1$ satisfy the invariant
wave equation of the same type as
\er{MaxVacFull1ninshtrgravortghhghgjkgghklhjgkghghjjkjhjkkggjkhjkhjjhhfhjhkjhjjhkhbbgjlmkmmhzzzyyykkknnnNWBWHWPPNjjlllkkkkkkljkkjjkjkhjjhhgghgggghjjhgjgghint}
or
\er{MaxVacFull1ninshtrgravortghhghgjkgghklhjgkghghjjkjhjkkggjkhjkhjjhhfhjhkjhjjhkhbbgjlmkmmhzzzyyykkknnnNWBWHWPPNjjlllkkkkkkljkkjjkjkhjjhhgghgggghjjhgjgghjkjkjkint}.
However, as we can see from
\er{hjgjgjgkllujkjjjkhhjhljjkjkljjljjhhhjkhhugghjjhjhgjkhhhkhjkhhkhkkhint}
and
\er{hjgjgjgkllujkjjjkhhjhljjkjkljjljjhhhjkhhugghjjhjhgjkhhhkhhjjghjkhjhjkjkjkkint}
the characteristic parameter $c_{01}$ in the wave equation for the
divergence part of $\vec g_1$
\er{MaxVacFull1ninshtrgravortghhghgjkgghklhjgkghghjjkjhjkkggjkhjkhjjhhfhjhkjkhbbgjhzzzyyykkknnnNWBWHWPPNjjjljkjkkhhhhhhhhhhhhfnfffhuyuyuyutjhjhjhggjjfgfhfhjhhgjgjhkjhjhkhkhjjint}
differ from the characteristic parameter $c_{02}$ in the wave
equation for the curl part of $\vec g_1$
\er{MaxVacFullPPNmmmffffffiuiuhjuGGggFGjjkkjjffjjint}, i.e. the
divergence part and the curl parts of  $\vec g_1$ propagate as two
different waves with the different speeds. As it can be easily seen,
in the case of the flat waves in the resting body the divergence
part of $\vec g_1$ (which is curl-free) propagate as a
\underline{longitudinal} wave, similarly to the sound in fluid or
gas, and at the same time the curl part of $\vec g_1$ (which is
divergence-free) propagate as a \underline{transverse} wave.
Moreover, the longitudinal and transverse waves in an elastic body
propagate with two different speeds.

\subsection{Macroscopic Electrodynamics in the presence of dielectric and/or magnetic
mediums} Consider system  \er{MaxVacFull1bjkgjhjhgjaaaint}
in some inertial or non-inertial cartesian coordinate system inside
a dielectric and/or magnetic medium:
\begin{equation}\label{MaxVacFullnnnnGGint}
\begin{cases}
curl_{\vec x} \vec H_0= \frac{4\pi}{c}\left(\vec j+\vec
j_{mp}\right)+ \frac{1}{c}\frac{\partial \vec D_0}{\partial t}
\\
div_{\vec x} \vec D_0= 4\pi\left(\rho+\rho_p\right)
\\
curl_{\vec x} \vec E+\frac{1}{c}\frac{\partial \vec B}{\partial t}=0
\\
div_{\vec x} \vec B= 0,
\end{cases}
\end{equation}
where $\vec E$ is the electric field, $\vec B$ is the magnetic
field, $\vec v:=\vec v(\vec x,t)$ is the vectorial gravitational
potential, $\rho$ is the average (macroscopic) charge density,
$\rho_p$ is the density of the charge of polarization, $\vec j$ is
the average (macroscopic) current density, $\vec j_{mp}$ is the
density of the current of polarization and magnetization and
\begin{equation}\label{MaxVacFullnnnngkggjkklhGGint}
\vec D_0:=\vec E+\frac{1}{c}\,\vec v\times \vec
B\quad\text{and}\quad
\vec H_0:=\vec B+\frac{1}{c}\,\vec v\times \vec D_0.
\end{equation}
Consistently with subsection \ref{huugftdtdtddtdt}, it is well known
from the Lorentz theory that in the case of a moving
dielectric/magnetic medium
\begin{equation}\label{PolarGGint}
\rho_p=-div_{\vec x} \vec P\quad\text{and}\quad
\vec j_{mp}=\frac{\partial \vec P}{\partial t}+c\, curl_{\vec x}
\vec M\,,
\end{equation}
where $\vec P:\R^3\times[0,+\infty)\to\R^3$ is the field of
polarization, $\vec M:\R^3\times[0,+\infty)\to\R^3$ is the field of
magnetization and $\vec u:=\vec u(\vec x,t)$ is the field of
velocities of the dielectric medium, see Proposition
\ref{hjghfhfhf11} for the proof of \er{PolarGGint}. Thus, if we
consider
\begin{equation}\label{OprdddGGint}
\vec D:=\vec D_0+4\pi \vec P=\vec E+\frac{1}{c}\,\vec v\times \vec
B+4\pi \vec P,
\end{equation}
and
\begin{equation}\label{Oprddd1GGint}
\vec H:=\vec H_0-4\pi \vec M
=\vec B
+\frac{1}{c}\,\vec
v\times \vec E+\frac{1}{c}\,\vec v\times\left(\frac{1}{c}\,\vec
v\times \vec B\right)-4\pi \vec M,
\end{equation}
we obtain the usual Maxwell equations of the form:
\begin{equation}\label{MaxMedFullGGint}
\begin{cases}
curl_{\vec x} \vec H= \frac{4\pi}{c}\vec j+
\frac{1}{c}\frac{\partial \vec D}{\partial t}
\\
div_{\vec x} \vec D= 4\pi\rho
\\
curl_{\vec x} \vec E+\frac{1}{c}\frac{\partial \vec B}{\partial t}=
0
\\
div_{\vec x} \vec B=0.
%
\end{cases}
\end{equation}
We call $\vec D$ by the electric displacement field and $\vec H$ by
the $\vec H$-magnetic field in a medium.

Next,
in section \ref{DMPGG} we prove that the laws of transformation of
electromagnetic fields in dielectric/magnetic medium, under the
change of non-inertial cartesian coordinate system of the form
\er{noninchgravortbstrjgghguittu2int}, are exactly the same as
\er{yuythfgfyftydtydtydtyddyyyhhddhhhredPPN111hgghjgint} in the
vacuum, i.e. having the form of
\begin{equation}\label{guigikvbvbggjklhjkkgjgGGGGint}
\begin{cases}
\vec D'=A(t)\cdot \vec D\\
\vec B'=A(t)\cdot\vec B\\
\vec E'=A(t)\cdot\vec E-\frac{1}{c}\,\left(\frac{dA}{dt}(t)\cdot\vec
x+\frac{d\vec
z}{dt}(t)\right)\times \left(A(t)\cdot\vec B\right)\\
\vec H'=A(t)\cdot\vec H+\frac{1}{c}\,\left(\frac{dA}{dt}(t)\cdot\vec
x+\frac{d\vec z}{dt}(t)\right)\times \left(A(t)\cdot\vec D\right),
\end{cases}
\end{equation}
provided that
\begin{equation}\label{guigikvbvbGGint}
\begin{cases}
\vec u'=A(t)\cdot \vec u+\frac{dA}{dt}(t)\cdot\vec x+\frac{d\vec z}{dt}(t)\\
\vec v'=A(t)\cdot \vec v+\frac{dA}{dt}(t)\cdot\vec x+\frac{d\vec
z}{dt}(t)\\
\vec P'=A(t)\cdot\vec P,\\
\vec M'=A(t)\cdot\vec M-\frac{1}{c}\,\left(\frac{dA}{dt}(t)\cdot\vec
x+\frac{d\vec z}{dt}(t)\right)\times \left(A(t)\cdot\vec P\right),
\end{cases}
\end{equation}
(see subsection \ref{huugftdtdtddtdt} for the justification of the
last two equalities in \er{guigikvbvbGGint}). Moreover, as before,
the Maxwell equations in the medium of the form \er{MaxMedFullGGint}
stays invariant under the change of inertial or non-inertial
cartesian coordinate system, provided we have
\er{guigikvbvbggjklhjkkgjgGGGGint} and
\begin{equation}\label{hjjgkoooint}
\vec j'=A(t)\cdot \vec j+\rho\,\frac{d A}{dt}(t)\cdot\vec x+\rho\,
\frac{d\vec z}{dt}(t).
\end{equation}

Next, one can show that in the case of simplest moving homogenous
isotropic nonmagnetic dielectric medium we have
\begin{equation}\label{EBDHTrans444knbuihuigGGojkjkuukuint}
\begin{cases}
\vec P=\frac{n^2-1}{4\pi}\left(\vec E+\frac{1}{c}\,\vec u\times \vec B\right),\\
\vec M=-\frac{1}{c}\,\frac{n^2-1}{4\pi}\,\vec u\times \left(\vec
E+\frac{1}{c}\,\vec u\times \vec B\right)=-\frac{1}{c}\vec u\times
\vec P,
\end{cases}
\end{equation}
where $n$ is a material coefficient (not necessary constant), called
\underline{refraction index} and $\vec u$ is the velocity of the
medium. Moreover, in the case of simplest homogenous isotropic
dielectric medium with certain magnetic properties one generalizes
\er{EBDHTrans444knbuihuigGGojkjkuukuint} by the following:
\begin{equation}\label{EBDHTrans444knbuihuigGGint}
\begin{cases}
\vec P=\frac{n^2-1}{4\pi}\left(\vec E+\frac{1}{c}\,\vec u\times \vec B\right)-\kappa\vec D,\\
\vec M=-\frac{1}{c}\,\frac{n^2-1}{4\pi}\,\vec u\times \left(\vec
E+\frac{1}{c}\,\vec u\times \vec B\right)+\kappa\vec
H=-\frac{1}{c}\vec u\times \vec P+\kappa\left(\vec
H-\frac{1}{c}\,\vec u\times \vec D\right),
\end{cases}
\end{equation}
where $n$ is the refraction index and $\kappa$ is an additional
material coefficient (not necessary constant), see subsection
\ref{hhhhhhijh} for the details. In the case of nonmagnetic
dielectric we just put $\kappa=0$ and
\er{EBDHTrans444knbuihuigGGint} becomes to be the same as
\er{EBDHTrans444knbuihuigGGojkjkuukuint}.
Then, using \er{guigikvbvbGGint} and
\er{guigikvbvbggjklhjkkgjgGGGGint}, it can be easily seen that the
laws in \er{EBDHTrans444knbuihuigGGint} are invariant under the
changes of inertial or non-inertial cartesian coordinate system.
Next, denoting
\begin{equation}\label{hkjhjhint}
\kappa_0=\frac{1}{1+4\pi
\kappa}\quad\text{and}\quad\gamma_0=\frac{1+4\pi\kappa}{n^2}=\frac{1}{\kappa_0n^2}\quad\text{so
that}\quad n=\frac{1}{\sqrt{\gamma_0\kappa_0}}
\end{equation}
and defining the speed-like vector field
\begin{equation}\label{OprdddsimGGffyhjyhhzzint}
\vec {\tilde u}:=\left(\frac{1}{n^2}\vec
v+\left(1-\frac{1}{n^2}\right)\,\vec
u\right)=\left(\gamma_0\kappa_0\vec
v+\left(1-\gamma_0\kappa_0\right)\,\vec u\right),
\end{equation}
by plugging \er{EBDHTrans444knbuihuigGGint} into \er{OprdddGGint}
and \er{Oprddd1GGint} we deduce
\begin{equation}\label{OprdddsimGGint}
\vec E=\gamma_0\vec D-\frac{1}{c}\,\vec {\tilde u}\times \vec B,
\end{equation}
and
\begin{equation}\label{Oprddd1simGGint}
\vec H=\kappa_0\vec B+\frac{1}{c}\,\vec {\tilde u}\times \vec
D+\frac{\kappa_0(1-\kappa_0\gamma_0)}{c^2}\,(\vec u-\vec
v)\times\left(\left(\vec u-\vec v\right)\times \vec B\right),
\end{equation}
where we call $\gamma_0$ and $\kappa_0$ dielectric and magnetic
permeability of the medium, see subsection \ref{hhhhhhijh} for the
details. Thus by \er{MaxMedFullGGint},
\er{OprdddsimGGffyhjyhhzzint}, \er{OprdddsimGGint} and
\er{Oprddd1simGGint} we have
\begin{equation}\label{MaxMedFullGGffykkzzkk}
\begin{cases}
curl_{\vec x} \vec H=\frac{4\pi}{c}\vec j+
\frac{1}{c}\frac{\partial \vec D}{\partial t},\\
div_{\vec x} \vec D=4\pi\rho,\\
curl_{\vec x} \vec E+\frac{1}{c}\frac{\partial \vec B}{\partial t}=0,\\
div_{\vec x} \vec B=0,\\
\vec E=\gamma_0\vec D-\frac{1}{c}\,\vec {\tilde u}\times \vec B,\\
\vec H=\kappa_0\vec B+\frac{1}{c}\,\vec {\tilde u}\times \vec
D+\frac{\kappa_0(1-\kappa_0\gamma_0)}{c^2}\,(\vec u-\vec
v)\times\left(\left(\vec
u-\vec v\right)\times \vec B\right),\\
\vec {\tilde u}:=\left(\kappa_0\gamma_0\vec
v+(1-\kappa_0\gamma_0)\vec u\right),
\end{cases}
\end{equation}
where $\vec {\tilde u}$ is a speed-like vector field that we call
the optical displacement of the moving medium. Note that for the
case $\gamma_0=1$ and $\kappa_0=1$, the system
\er{MaxMedFullGGffykkzzkk} is exactly the same as the corresponding
system in the vacuum. The equations in \er{MaxMedFullGGffykkzzkk}
take much simpler forms in the case where the quantity
\begin{equation}\label{OprdddsimGGffyhjyhhtygrffgfzzkk}
\frac{|\kappa_0||1-\kappa_0\gamma_0|\cdot|\vec u-\vec v|^2}{c^2}\ll
1
\end{equation}
is negligible, that happens if either the absolute value of the
difference between the medium velocity and vectorial gravitational
potential is much less then the constant $c$ or
$\gamma_0\kappa_0=\frac{1}{n^2}$ is close to the value $1$. Indeed,
in this case, instead of \er{OprdddsimGGint} and
\er{Oprddd1simGGint} we obtain the following relations:
\begin{align}\label{OprdddsimsimsimGGint}
\vec E=\gamma_0\vec D-\frac{1}{c}\,\vec {\tilde u}\times \vec B,\\
\label{Oprddd1simsimsimGGint}
\vec H=\kappa_0\vec B+\frac{1}{c}\,\vec {\tilde u}\times \vec D.
\end{align}
As a consequence we obtain the full system of Maxwell equations in
the medium:
\begin{equation}\label{MaxMedFullGGffykkhjhhzzkk}
\begin{cases}
curl_{\vec x} \vec H=\frac{4\pi}{c}\vec j+
\frac{1}{c}\frac{\partial \vec D}{\partial t},\\
div_{\vec x} \vec D=4\pi\rho,\\
curl_{\vec x} \vec E+\frac{1}{c}\frac{\partial \vec B}{\partial t}=0,\\
div_{\vec x} \vec B=0,\\
\vec E=\gamma_0\vec D-\frac{1}{c}\,\vec {\tilde u}\times \vec B,\\
\vec H=\kappa_0\vec B+\frac{1}{c}\,\vec {\tilde u}\times \vec D,\\
\vec {\tilde u}=\left(\kappa_0\gamma_0\vec
v+(1-\kappa_0\gamma_0)\vec u\right),
\end{cases}
\end{equation}
where $\vec {\tilde u}$ is the speed-like vector field and
$\gamma_0$ and $\kappa_0$ are dielectric and magnetic permeability
of the medium. Note that \er{MaxMedFullGGffykkhjhhzzkk} is analogous
to the system of Maxwell equations in the vacuum and it is also
invariant under the change of inertial or non-inertial cartesian
coordinate system, provided that under this transformation we have
\er{guigikvbvbGGint}.

Next, in the case of the simplest anisotropic dielectric and/or
magnetic, it is well known that we have the following generalization
of \er{EBDHTrans444knbuihuigGGint}:
\begin{equation}\label{EBDHTrans444knbuihuigGGaaint}
\begin{cases}
\vec P=\Gamma\cdot\left(\vec E+\frac{1}{c}\,\vec u\times \vec B\right)-\Upsilon\cdot\vec D,\\
\vec M=
-\frac{1}{c}\vec u\times \vec P+\Upsilon\cdot\left(\vec
H-\frac{1}{c}\,\vec u\times \vec D\right),
\end{cases}
\end{equation}
where $\Gamma\in \mathbb{R}^{3\times 3}$ and $\Upsilon\in
\mathbb{R}^{3\times 3}$ are matrix-valued fields.
Then, using \er{guigikvbvbGGint} and
\er{guigikvbvbggjklhjkkgjgGGGGint} it can be easily seen that the
laws in \er{EBDHTrans444knbuihuigGGaaint} are invariant under the
changes of inertial or non-inertial cartesian coordinate system,
provided that $\Gamma$ and $\Upsilon$ are \underline{proper} matrix
fields (see Definition \ref{bggghghgjint}).

Next, it is well known that the Ohm's Law in a conducting medium has
the form
\begin{equation}\label{vjhfhjtjhjhuyyiyGGint}
\vec j-\rho\,\vec u=\varepsilon\left(\vec E+\frac{1}{c}\,\vec
u\times \vec B\right),
\end{equation}
where $\vec u$ is the velocity of the medium and $\varepsilon$ is a
material coefficient. As before, using
\er{guigikvbvbGGint} and \er{guigikvbvbggjklhjkkgjgGGGGint}, it can
be easily seen that the Ohm's Law is invariant under the changes of
inertial or non-inertial cartesian coordinate system.

Furthermore, it is well known that in the case of the strong
magnetic field the modification of the the Ohm's Law including the
Hall effect has the following form:
\begin{equation}\label{vjhfhjtjhjhuyyiyGGjljint}
\vec j-\rho\,\vec u=\varepsilon\left(\vec E+\frac{1}{c}\,\vec
u\times \vec B\right)-\varsigma\left(\vec j-\rho\,\vec
u\right)\times \vec B,
\end{equation}
where $\varsigma$ is a material coefficient. Then, as before,  using
\er{guigikvbvbGGint} and \er{guigikvbvbggjklhjkkgjgGGGGint}, it can
be easily seen that the generalized Ohm's Law
\er{vjhfhjtjhjhuyyiyGGjljint}, including the Hall effect, is
invariant under the changes of inertial or non-inertial cartesian
coordinate system.

Next, it is well known that, in the case of the anisotropic
conducting medium, the Ohm's Law has the following form,
generalizing \er{vjhfhjtjhjhuyyiyGGint}:
\begin{equation}\label{vjhfhjtjhjhuyyiyGGaaint}
\vec j-\rho\,\vec u=\Xi\cdot\left(\vec E+\frac{1}{c}\,\vec u\times
\vec B\right),
\end{equation}
where $\Xi\in \mathbb{R}^{3\times 3}$ is a matrix-valued field. As
before, using
\er{guigikvbvbGGint} and \er{guigikvbvbggjklhjkkgjgGGGGint}, it can
be easily seen that \er{vjhfhjtjhjhuyyiyGGaaint} is invariant under
the changes of inertial or non-inertial cartesian coordinate system,
provided that $\Xi$ is a \underline{proper} matrix field.

Finally, the generalization of \er{vjhfhjtjhjhuyyiyGGjljint} to
anisotropic mediums is the following:
\begin{equation}\label{vjhfhjtjhjhuyyiyGGjljaaint}
\vec j-\rho\,\vec u=\Xi\cdot\left(\vec E+\frac{1}{c}\,\vec u\times
\vec B\right)-\Pi\cdot\left(\left(\vec j-\rho\,\vec u\right)\times
\vec B\right),
\end{equation}
where $\Xi\in \mathbb{R}^{3\times 3}$ and $\Pi\in
\mathbb{R}^{3\times 3}$  are matrix-valued fields. As before,  using
\er{guigikvbvbGGint} and \er{guigikvbvbggjklhjkkgjgGGGGint}, it can
be easily seen that \er{vjhfhjtjhjhuyyiyGGjljaaint} is invariant
under the changes of inertial or non-inertial cartesian coordinate
system, provided that $\Xi$ and $\Pi$ are \underline{proper} matrix
fields.

\subsubsection{Optical dispersion in moving mediums}
Remind that if  $U$ is a real valued field, then we can write the
field $U$ as a Fourier's Transform on the time variable:
\begin{equation}\label{MaxVacFullPPNmmmffffffhhtygghGGGGyuhggghghghffggggjj11int}
U(\vec x,t)=Re\left\{2\int_{0}^{+\infty} \hat U(\vec
x,\omega)e^{i\omega t}d\omega\right\}\quad\text{where}\quad\hat
U(\vec x,\omega):=\frac{1}{2\pi}\int_{-\infty}^{+\infty} U(\vec
x,t)e^{-i\omega t}dt\,.
\end{equation}
Thus, since $\vec E$, $\vec B$, $\vec D$ and $\vec H$ are real
vectors, we can write them as a Fourier's Transform on the time
variable:
\begin{align}\label{MaxVacFullPPNmmmffffffhhtygghGGGGyuhggghghghffggggjj11jhhhjj11hgjgj11ghfh11int}
\vec E(\vec x,t)=Re\left\{2\int_{0}^{+\infty}  {\vec {\hat
E}(\omega,\vec x)}e^{i\omega t}d\omega\right\}\quad\text{where}\quad
{\vec {\hat E}(\omega,\vec
x)}:=\frac{1}{2\pi}\int_{-\infty}^{+\infty}
\vec E(\vec x,t)e^{-i \omega t}dt\,,\\
\label{MaxVacFullPPNmmmffffffhhtygghGGGGyuhggghghghffggggjj11jhhhjj11ihhgjjj11mbbmbm11int}
\vec B(\vec x,t)=Re\left\{2\int_{0}^{+\infty}  {\vec {\hat
B}(\omega,\vec x)}e^{i\omega t}d\omega\right\}\quad\text{where}\quad
{\vec {\hat B}(\omega,\vec
x)}:=\frac{1}{2\pi}\int_{-\infty}^{+\infty} \vec B(\vec x,t)e^{-i
\omega
t}dt\,,\\
\label{MaxVacFullPPNmmmffffffhhtygghGGGGyuhggghghghffggggjj11jhhhjj11hgjgj11ghfh11Dint}
\vec D(\vec x,t)=Re\left\{2\int_{0}^{+\infty}  {\vec {\hat
D}(\omega,\vec x)}e^{i\omega t}d\omega\right\}\quad\text{where}\quad
{\vec {\hat D}(\omega,\vec
x)}:=\frac{1}{2\pi}\int_{-\infty}^{+\infty}
\vec D(\vec x,t)e^{-i \omega t}dt\,,\\
\label{MaxVacFullPPNmmmffffffhhtygghGGGGyuhggghghghffggggjj11jhhhjj11ihhgjjj11mbbmbm11Hint}
\vec H(\vec x,t)=Re\left\{2\int_{0}^{+\infty}  {\vec {\hat
H}(\omega,\vec x)}e^{i\omega t}d\omega\right\}\quad\text{where}\quad
{\vec {\hat H}(\omega,\vec
x)}:=\frac{1}{2\pi}\int_{-\infty}^{+\infty} \vec H(\vec x,t)e^{-i
\omega t}dt\,,
\end{align}
where given $\vec a=(a_1,a_2,a_3)\in \mathbb{C}^3$ we denote
$Re\,\{\vec a\}=\left(Re\,\{a_1\},Re\,\{a_2\},Re\,\{a_3\}\right)\in
\mathbb{R}^3$. Then, it is well known that in the case of the
simplest optical dispersion in the resting medium with $\vec u\equiv
0$ we have the following generalization of
\er{EBDHTrans444knbuihuigGGint}:
\begin{equation}\label{EBDHTrans444knbuihuigGGjj11jgggggjj11hhhjjj11gvghhghgjj11jgjggj11hghffh11int}
\begin{cases}
\vec P\left(\vec x,t\right)=Re\left\{2\int_{0}^{+\infty}
{\frac{n^2\left(\omega,\vec x,t\right)-1}{4\pi}\,\vec {\hat
E}\left(\omega,\vec
 x\right)}e^{i\omega t}d\omega\right\}-Re\left\{2\int_{0}^{+\infty}
{\kappa\left(\omega,\vec x,t\right)\vec {\hat D}\left(\omega,\vec
x\right)}e^{i\omega t}d\omega\right\},\\
\vec M\left(\vec x,t\right)=Re\left\{2\int_{0}^{+\infty}
{\kappa\left(\omega,\vec x,t\right)\vec {\hat H}\left(\omega,\vec
x\right)}e^{i\omega t}d\omega\right\},
\end{cases}
\end{equation}
where $\vec {\hat E}(\omega,\vec x)$, $\vec {\hat B}(\omega,\vec
x)$, $\vec {\hat D}(\omega,\vec x)$ and $\vec {\hat H}(\omega,\vec
x)$ are Fourier's Transforms by time of $\vec E(\vec x,t)$, $\vec
B(\vec x,t)$, $\vec D(\vec x,t)$ and $\vec H(\vec x,t)$, given by
the second equalities of
\er{MaxVacFullPPNmmmffffffhhtygghGGGGyuhggghghghffggggjj11jhhhjj11hgjgj11ghfh11int},
\er{MaxVacFullPPNmmmffffffhhtygghGGGGyuhggghghghffggggjj11jhhhjj11ihhgjjj11mbbmbm11int},
\er{MaxVacFullPPNmmmffffffhhtygghGGGGyuhggghghghffggggjj11jhhhjj11hgjgj11ghfh11Dint}
and
\er{MaxVacFullPPNmmmffffffhhtygghGGGGyuhggghghghffggggjj11jhhhjj11ihhgjjj11mbbmbm11Hint}
and the quantities $n\left(\omega,\vec x,t\right)$ and
$\kappa\left(\omega,\vec x,t\right)$ in
\er{EBDHTrans444knbuihuigGGjj11jgggggjj11hhhjjj11gvghhghgjj11jgjggj11hghffh11int}
are assumed to be \underline{complex} in general and assumed to
depend on $\omega$ in addition to the dependence on $(\vec x,t)$.

We would like to obtain the law being an analog of
\er{EBDHTrans444knbuihuigGGjj11jgggggjj11hhhjjj11gvghhghgjj11jgjggj11hghffh11int}
in the case of moving medium (i.e. $\vec u\neq0$), which is
invariant under the change of inertial or non-inertial cartesian
coordinate system. In addition to the invariance under the change of
cartesian coordinate systems and the equivalence to
\er{EBDHTrans444knbuihuigGGjj11jgggggjj11hhhjjj11gvghhghgjj11jgjggj11hghffh11int}
in the particular case $\vec u\equiv 0$, we also need to assume that
in the case of an arbitrary moving transparent medium without
dispersion our law is equivalent to \er{EBDHTrans444knbuihuigGGint}.

Then, in some cartesian coordinate system consider a motion of some
continuum medium occupying a region $\Omega\subset\mathbb{R}^3$ at
some fixed instant of time $t=t_0$ and having the velocity field
$\vec u:=\vec u(\vec x,t)$. Next, as before, let $\vec r(t,\vec
y):\mathbb{R}\times\Omega\to\mathbb{R}^3$ be a solution of
\er{hjgghhghhffhkkint} i.e. it satisfies the following initial value
problem for an ordinary differential equation:
\begin{equation}\label{hjgghhghhffhkkjj11int}
\begin{cases}
\frac{\partial\vec r}{\partial t}(t,\vec y)=\vec u\left(\vec
r(t,\vec y),t\right)\quad\quad\forall t\in\mathbb{R},\;\;\forall\vec
y\in\Omega\\
\vec r(t_0,\vec y)=\vec y\quad\quad\forall\vec y\in\Omega.
\end{cases}
\end{equation}
Then, clearly $\vec x=\vec r(t,\vec y)$ stands for the spatial
coordinates at the instant of time $t$ of the parcel of continuum,
having initial coordinates $\vec y$. Thus, as before,
we deduce that $det\left\{d_{\vec y}\vec r(t,\vec y)\right\}\neq 0$
for every instant of time $t$ and so, for the given instant of time
$t$ the mapping $\vec x=\vec r(t,\vec y)$ is locally invertible i.e.
the equation $\vec x=\vec r(t,\vec y)$ can be resolved in $\vec y$.
Thus there exists a regular mapping $\vec y=\vec f(\vec x,t)$, such
that
\begin{equation}\label{hjgghhghhffhkkihiyjhjkhjhjhjj11int}
\vec f\left(\vec r(t,\vec y),t\right)=\vec y\quad\forall \vec
y\in\Omega\quad\quad\text{and}\quad\quad\vec r\left(t,\vec f(\vec
x,t)\right)=\vec x.
\end{equation}
Then, clearly $\vec y=\vec f(\vec x,t)$ stands for the initial
coordinates at the time $t_0$ of the parcel of continuum, having
coordinates $\vec x$ at the instant of time $t$. Next $\vec f(\vec
x,t)$ satisfies:
\begin{equation}\label{hjgghhghhffhkkmkljjnnnnjj11int}
\begin{cases}
\frac{\partial\vec f}{\partial t}\left(\vec x,t\right)+d_{\vec
x}\vec f\left(\vec x,t\right)\cdot\vec u\left(\vec x,t\right)=0
\\
\vec f(\vec x,t_0)=\vec x.
\end{cases}
\end{equation}

Next consider,
\begin{equation}\label{EBDHTrans444knbuihuigGGjj11ghthgtjkjj11int}
\begin{cases}
\vec E^*(t,\vec y):=\left\{d_{\vec y}\vec r\left(t,\vec
y\right)\right\}^T\cdot\left(\vec E\left(\vec r(t,\vec
y),t\right)+\frac{1}{c}\,\vec u\left(\vec r(t,\vec y),t\right)\times
\vec B\left(\vec r(t,\vec y),t\right)\right)\quad\quad\forall
t\in\mathbb{R},\;\;\forall\vec
y\in\Omega,\\
\vec B^*(t,\vec y):=\left\{d_{\vec y}\vec r\left(t,\vec
y\right)\right\}^T\cdot\vec B\left(\vec r(t,\vec
y),t\right)\quad\quad\forall t\in\mathbb{R},\;\;\forall\vec
y\in\Omega,\\
\vec D^*(t,\vec y):=\left\{d_{\vec y}\vec r\left(t,\vec
y\right)\right\}^T\cdot\vec D\left(\vec r(t,\vec
y),t\right)\quad\quad\forall t\in\mathbb{R},\;\;\forall\vec
y\in\Omega,\\
\vec H^*(t,\vec y):=\left\{d_{\vec y}\vec r\left(t,\vec
y\right)\right\}^T\cdot\left(\vec H\left(\vec r(t,\vec
y),t\right)-\frac{1}{c}\,\vec u\left(\vec r(t,\vec y),t\right)\times
\vec D\left(\vec r(t,\vec y),t\right)\right)\quad\quad\forall
t\in\mathbb{R},\;\;\forall\vec y\in\Omega,
\end{cases}
\end{equation}
where $\vec r(t,\vec y)$ is given by \er{hjgghhghhffhkkjj11}.
Then, since $\vec E^*$, $\vec B^*$, $\vec D^*$ and $\vec H^*$ are
real vectors, we can write them as a Fourier's Transform on the time
variable:
\begin{align}
\label{MaxVacFullPPNmmmffffffhhtygghGGGGyuhggghghghffggggjj11jhhhjj11int}
\vec E^*(t,\vec y)=Re\left\{2\int_{0}^{+\infty}  {\vec {\hat
E^*}(\tilde\omega,\vec y)}e^{i\tilde \omega
t}d\tilde\omega\right\}\quad\text{where}\quad {\vec {\hat
E^*}(\tilde\omega,\vec y)}:=\frac{1}{2\pi}\int_{-\infty}^{+\infty}
\vec E^*(t,\vec y)e^{-i\tilde \omega t}dt\,,\\
\label{MaxVacFullPPNmmmffffffhhtygghGGGGyuhggghghghffggggjj11jhhhjj11ihhgjjj11int}
\vec B^*(t,\vec y)=Re\left\{2\int_{0}^{+\infty}  {\vec {\hat
B^*}(\tilde\omega,\vec y)}e^{i\tilde \omega
t}d\tilde\omega\right\}\quad\text{where}\quad {\vec {\hat
B^*}(\tilde\omega,\vec y)}:=\frac{1}{2\pi}\int_{-\infty}^{+\infty}
\vec B^*(t,\vec y)e^{-i\tilde \omega t}dt\,,\\
\label{MaxVacFullPPNmmmffffffhhtygghGGGGyuhggghghghffggggjj11jhhhjj11Dint}
\vec D^*(t,\vec y)=Re\left\{2\int_{0}^{+\infty}  {\vec {\hat
D^*}(\tilde\omega,\vec y)}e^{i\tilde \omega
t}d\tilde\omega\right\}\quad\text{where}\quad {\vec {\hat
D^*}(\tilde\omega,\vec y)}:=\frac{1}{2\pi}\int_{-\infty}^{+\infty}
\vec D^*(t,\vec y)e^{-i\tilde \omega t}dt\,,\\
\label{MaxVacFullPPNmmmffffffhhtygghGGGGyuhggghghghffggggjj11jhhhjj11ihhgjjj11Hint}
\vec H^*(t,\vec y)=Re\left\{2\int_{0}^{+\infty}  {\vec {\hat
H^*}(\tilde\omega,\vec y)}e^{i\tilde \omega
t}d\tilde\omega\right\}\quad\text{where}\quad {\vec {\hat
H^*}(\tilde\omega,\vec y)}:=\frac{1}{2\pi}\int_{-\infty}^{+\infty}
\vec H^*(t,\vec y)e^{-i\tilde \omega t}dt\,.
\end{align}
Then we write the generalization of
\er{EBDHTrans444knbuihuigGGjj11jgggggjj11hhhjjj11gvghhghgjj11jgjggj11hghffh11int}
as:
\begin{multline}\label{EBDHTrans444knbuihuigGGjj11jgggggjj11hhhjjj11int}
\vec P\left(\vec r(t,\vec y),t\right)= Re\left\{2\int_{0}^{+\infty}
{\frac{n^2\left(\tilde\omega,\vec r(t,\vec
y),t\right)-1}{4\pi}\,\left(\left(\left\{d_{\vec y}\vec
r\left(t,\vec y\right)\right\}^T\right)^{-1}\cdot\vec {\hat
E^*}\left(\omega,\vec
 x\right)\right)}e^{i\omega t}d\omega\right\}\\-Re\left\{2\int_{0}^{+\infty}
{\kappa\left(\tilde\omega,\vec r(t,\vec
y),t\right)\left(\left(\left\{d_{\vec y}\vec r\left(t,\vec
y\right)\right\}^T\right)^{-1}\cdot\vec {\hat D^*}\left(\omega,\vec
x\right)\right)}e^{i\omega t}d\omega\right\}\\
\quad\text{and}\quad\quad\quad \vec M\left(\vec r(t,\vec
y),t\right)+\frac{1}{c}\vec u\left(\vec
r(t,\vec y),t\right)\times \vec P\left(\vec r(t,\vec y),t\right)=\\
Re\left\{2\int_{0}^{+\infty} {\kappa\left(\tilde\omega,\vec r(t,\vec
y),t\right)\left(\left(\left\{d_{\vec y}\vec r\left(t,\vec
y\right)\right\}^T\right)^{-1}\cdot\vec {\hat H^*}(\tilde\omega,\vec
y)\right)}e^{i\tilde \omega t}d\tilde\omega\right\},
\end{multline}
i.e.
\begin{multline}\label{EBDHTrans444knbuihuigGGjj11jgggggjj11hhhjjj11gvghhghgjj11int}
\vec P\left(\vec x,t\right)=Re\left\{2\int_{0}^{+\infty}
{\frac{n^2\left(\tilde\omega,\vec
x,t\right)-1}{4\pi}\,\left(\left\{d_{\vec x}\vec f\left(\vec
x,t\right)\right\}^T\cdot\vec {\hat E^*}\left(\tilde\omega,\vec
f\left(\vec x,t\right)\right)\right)}e^{i\tilde \omega
t}d\tilde\omega\right\}\\-Re\left\{2\int_{0}^{+\infty}
{\kappa\left(\tilde\omega,\vec x,t\right)\left(\left\{d_{\vec x}\vec
f\left(\vec x,t\right)\right\}^T\cdot\vec {\hat
D^*}\left(\tilde\omega,\vec f\left(\vec
x,t\right)\right)\right)}e^{i\tilde \omega t}d\tilde\omega\right\}\quad\quad\text{and} \\
\vec M\left(\vec x,t\right)+\frac{1}{c}\vec u(\vec x,t)\times \vec
P(\vec x,t)=Re\left\{2\int_{0}^{+\infty}
{\kappa\left(\tilde\omega,\vec x,t\right)\left(\left\{d_{\vec x}\vec
f\left(\vec x,t\right)\right\}^T\cdot\vec {\hat
H^*}\left(\tilde\omega,\vec f\left(\vec
x,t\right)\right)\right)}e^{i\tilde \omega t}d\tilde\omega\right\}.
\end{multline}
The quantities $n\left(\tilde\omega,\vec x,t\right)$ and
$\kappa\left(\tilde\omega,\vec x,t\right)$ in
\er{EBDHTrans444knbuihuigGGjj11jgggggjj11hhhjjj11int} and
\er{EBDHTrans444knbuihuigGGjj11jgggggjj11hhhjjj11gvghhghgjj11int}
are assumed to be \underline{complex} in general and assumed to
depend on $\tilde\omega$ in addition to the dependence on $(\vec
x,t)$.

In particular, in the case $\vec u\equiv 0$ we clearly have $\vec
r(t,\vec y)\equiv\vec y$, $\vec f\left(\vec x,t\right)\equiv\vec x$
and $\vec E^*(t,\vec x)=\vec E(\vec x,t)$, $\vec B^*(t,\vec x)=\vec
B(\vec x,t)$, $\vec D^*(t,\vec x)=\vec D(\vec x,t)$, $\vec
H^*(t,\vec x)=\vec H(\vec x,t)$ and therefore,
\er{EBDHTrans444knbuihuigGGjj11jgggggjj11hhhjjj11int} and
\er{EBDHTrans444knbuihuigGGjj11jgggggjj11hhhjjj11gvghhghgjj11int}
coincide with
\er{EBDHTrans444knbuihuigGGjj11jgggggjj11hhhjjj11gvghhghgjj11jgjggj11hghffh11int}.

On the other hand, in subsection \ref{hugygyggygygyfgy} we prove
that in the case of an arbitrary moving transparent medium without
dispersion i.e. in the case where we assume that $n:=n\left(\vec
x,t\right)$ and $\kappa:=\kappa\left(\vec x,t\right)$ are
independent of the argument $\tilde\omega$ and moreover we assume
them to be \underline{real}, we can rewrite
\er{EBDHTrans444knbuihuigGGjj11jgggggjj11hhhjjj11gvghhghgjj11int} in
the simpler form:
\begin{equation}\label{EBDHTrans444knbuihuigGGjj11gjjggjg11ghhfgg11int}
\begin{cases}
\vec P(\vec x,t)=\frac{n^2\left(\vec x,t\right)-1}{4\pi}\,\left(\vec
E(\vec x,t)+\frac{1}{c}\,\vec u(\vec x,t)\times \vec B(\vec
x,t)\right)-\kappa(\vec x,t)\,\vec D(\vec
x,t),\\
\vec M(\vec x,t)=\kappa(\vec x,t)\,\left(\vec H(\vec
x,t)-\frac{1}{c}\,\vec u(\vec x,t)\times \vec D(\vec
x,t)\right)-\frac{1}{c}\vec u\left(\vec x,t\right)\times \vec
P\left(\vec x,t\right),
\end{cases}
\end{equation}
consistently with \er{EBDHTrans444knbuihuigGGint}.

Next, as before, assume that the change of some non-inertial
cartesian system $(*)$ of coordinates to another cartesian system
$(**)$ of coordinates is of the form:
\begin{equation*}
\begin{cases}
\vec x'=A(t)\cdot\vec x+\vec z(t),\\
t'=t,
\end{cases}
\end{equation*}
where $A(t)\in SO(3)$ is a rotation. Then, in subsection
\ref{hugygyggygygyfgy} we prove that the dispersion law, in the
forms of \er{EBDHTrans444knbuihuigGGjj11jgggggjj11hhhjjj11int} or
equivalently
\er{EBDHTrans444knbuihuigGGjj11jgggggjj11hhhjjj11gvghhghgjj11int},
is invariant under the change of inertial or non-inertial cartesian
coordinate system, provided we have
\begin{equation}\label{hjgjgjgkllujkjjjkhhjhljjkjhhjh11hkkhjhgj11jkjkjkgjg11hkhk4hjhjh11int}
\tilde\omega'=\tilde\omega,\quad n'\left(\tilde\omega',\vec
x',t'\right)=n\left(\tilde\omega,\vec x,t\right)\quad\text{and}\quad
\kappa'\left(\tilde\omega',\vec
x',t'\right)=\kappa\left(\tilde\omega,\vec x,t\right).
\end{equation}

Next, in subsection \ref{hugygyggygygyfgy} we prove the the right
hand side of \er{EBDHTrans444knbuihuigGGjj11jgggggjj11hhhjjj11int}
or \er{EBDHTrans444knbuihuigGGjj11jgggggjj11hhhjjj11gvghhghgjj11int}
is independent of the initial instant of time $t_0$ in
\er{hjgghhghhffhkkjj11int}.

Finally, we can write the following possible alternative to the
dispersion law in the form of
\er{EBDHTrans444knbuihuigGGjj11ghthgtjkjj11int} and
\er{EBDHTrans444knbuihuigGGjj11jgggggjj11hhhjjj11int} or
\er{EBDHTrans444knbuihuigGGjj11jgggggjj11hhhjjj11gvghhghgjj11int}.
We could consider,
\begin{equation}\label{EBDHTrans444knbuihuigGGjj11ghthgtjkjj11jj11int}
\begin{cases}
\vec E^*(t,\vec y):=\left\{d_{\vec y}\vec r\left(t,\vec
y\right)\right\}^{-1}\cdot\left(\vec E\left(\vec r(t,\vec
y),t\right)+\frac{1}{c}\,\vec u\left(\vec r(t,\vec y),t\right)\times
\vec B\left(\vec r(t,\vec y),t\right)\right)\quad\quad\forall
t\in\mathbb{R},\;\;\forall\vec
y\in\Omega,\\
\vec B^*(t,\vec y):=\left\{d_{\vec y}\vec r\left(t,\vec
y\right)\right\}^{-1}\cdot\vec B\left(\vec r(t,\vec
y),t\right)\quad\quad\forall t\in\mathbb{R},\;\;\forall\vec
y\in\Omega,\\
\vec D^*(t,\vec y):=\left\{d_{\vec y}\vec r\left(t,\vec
y\right)\right\}^{-1}\cdot\vec D\left(\vec r(t,\vec
y),t\right)\quad\quad\forall t\in\mathbb{R},\;\;\forall\vec
y\in\Omega,\\
\vec H^*(t,\vec y):=\left\{d_{\vec y}\vec r\left(t,\vec
y\right)\right\}^{-1}\cdot\left(\vec H\left(\vec r(t,\vec
y),t\right)-\frac{1}{c}\,\vec u\left(\vec r(t,\vec y),t\right)\times
\vec D\left(\vec r(t,\vec y),t\right)\right)\quad\quad\forall
t\in\mathbb{R},\;\;\forall\vec y\in\Omega,
\end{cases}
\end{equation}
where $\vec r(t,\vec y)$ is given by \er{hjgghhghhffhkkjj11}.
Then, since $\vec E^*$, $\vec D^*$, $\vec B^*$ and $\vec H^*$ are
real vectors, we could write them as a Fourier's Transform on the
time variable:
\begin{align}\label{MaxVacFullPPNmmmffffffhhtygghGGGGyuhggghghghffggggjj11jhhhjj11jj11int}
\vec E^*(t,\vec y)=Re\left\{2\int_{0}^{+\infty}  {\vec {\hat
E^*}(\tilde\omega,\vec y)}e^{i\tilde \omega
t}d\tilde\omega\right\}\quad\text{where}\quad {\vec {\hat
E^*}(\tilde\omega,\vec y)}:=\frac{1}{2\pi}\int_{-\infty}^{+\infty}
\vec E^*(t,\vec y)e^{-i\tilde \omega t}dt\,,\\
\label{MaxVacFullPPNmmmffffffhhtygghGGGGyuhggghghghffggggjj11jhhhjj11ihhgjjj11bhghhhg11int}
\vec B^*(t,\vec y)=Re\left\{2\int_{0}^{+\infty}  {\vec {\hat
B^*}(\tilde\omega,\vec y)}e^{i\tilde \omega
t}d\tilde\omega\right\}\quad\text{where}\quad {\vec {\hat
B^*}(\tilde\omega,\vec y)}:=\frac{1}{2\pi}\int_{-\infty}^{+\infty}
\vec B^*(t,\vec y)e^{-i\tilde \omega t}dt \,,\\
\label{MaxVacFullPPNmmmffffffhhtygghGGGGyuhggghghghffggggjj11jhhhjj11jj11Dint}
\vec D^*(t,\vec y)=Re\left\{2\int_{0}^{+\infty}  {\vec {\hat
D^*}(\tilde\omega,\vec y)}e^{i\tilde \omega
t}d\tilde\omega\right\}\quad\text{where}\quad {\vec {\hat
D^*}(\tilde\omega,\vec y)}:=\frac{1}{2\pi}\int_{-\infty}^{+\infty}
\vec D^*(t,\vec y)e^{-i\tilde \omega t}dt\,,\\
\label{MaxVacFullPPNmmmffffffhhtygghGGGGyuhggghghghffggggjj11jhhhjj11ihhgjjj11bhghhhg11Hint}
\vec H^*(t,\vec y)=Re\left\{2\int_{0}^{+\infty}  {\vec {\hat
H^*}(\tilde\omega,\vec y)}e^{i\tilde \omega
t}d\tilde\omega\right\}\quad\text{where}\quad {\vec {\hat
H^*}(\tilde\omega,\vec y)}:=\frac{1}{2\pi}\int_{-\infty}^{+\infty}
\vec H^*(t,\vec y)e^{-i\tilde \omega t}dt\,.
\end{align}
Then, alternatively to
\er{EBDHTrans444knbuihuigGGjj11ghthgtjkjj11int} and
\er{EBDHTrans444knbuihuigGGjj11jgggggjj11hhhjjj11int} we could
consider
\begin{multline}\label{EBDHTrans444knbuihuigGGjj11jgggggjj11hhhjjj11jj11int}
\vec P\left(\vec r(t,\vec y),t\right)= Re\left\{2\int_{0}^{+\infty}
{\frac{n^2\left(\tilde\omega,\vec r(t,\vec
y),t\right)-1}{4\pi}\,\left(d_{\vec y}\vec r\left(t,\vec
y\right)\cdot\vec {\hat E^*}\left(\omega,\vec
 x\right)\right)}e^{i\omega t}d\omega\right\}\\-Re\left\{2\int_{0}^{+\infty}
{\kappa\left(\tilde\omega,\vec r(t,\vec y),t\right)\left(d_{\vec
y}\vec r\left(t,\vec y\right)\cdot\vec {\hat D^*}\left(\omega,\vec
x\right)\right)}e^{i\omega t}d\omega\right\}\\
\quad\text{and}\quad\quad\quad \vec M\left(\vec r(t,\vec
y),t\right)+\frac{1}{c}\vec u\left(\vec
r(t,\vec y),t\right)\times \vec P\left(\vec r(t,\vec y),t\right)=\\
Re\left\{2\int_{0}^{+\infty} {\kappa\left(\tilde\omega,\vec r(t,\vec
y),t\right)\left(d_{\vec y}\vec r\left(t,\vec y\right)\cdot\vec
{\hat H^*}(\tilde\omega,\vec y)\right)}e^{i\tilde \omega
t}d\tilde\omega\right\},
\end{multline}
i.e.
\begin{multline}\label{EBDHTrans444knbuihuigGGjj11jgggggjj11hhhjjj11gvghhghgjj11jj11int}
\vec P\left(\vec x,t\right)=Re\left\{2\int_{0}^{+\infty}
{\frac{n^2\left(\tilde\omega,\vec
x,t\right)-1}{4\pi}\,\left(\left\{d_{\vec x}\vec f\left(\vec
x,t\right)\right\}^{-1}\cdot\vec {\hat E^*}\left(\tilde\omega,\vec
f\left(\vec x,t\right)\right)\right)}e^{i\tilde \omega
t}d\tilde\omega\right\}\\-Re\left\{2\int_{0}^{+\infty}
{\kappa\left(\tilde\omega,\vec x,t\right)\left(\left\{d_{\vec x}\vec
f\left(\vec x,t\right)\right\}^{-1}\cdot\vec {\hat
D^*}\left(\tilde\omega,\vec f\left(\vec
x,t\right)\right)\right)}e^{i\tilde \omega t}d\tilde\omega\right\}\quad\quad\text{and}\\
\vec M\left(\vec x,t\right)+\frac{1}{c}\vec u(\vec x,t)\times \vec
P(\vec x,t)=Re\left\{2\int_{0}^{+\infty}
{\kappa\left(\tilde\omega,\vec x,t\right)\left(\left\{d_{\vec x}\vec
f\left(\vec x,t\right)\right\}^{-1}\cdot\vec {\hat
H^*}\left(\tilde\omega,\vec f\left(\vec
x,t\right)\right)\right)}e^{i\tilde \omega t}d\tilde\omega\right\}.
\end{multline}
Then, exactly in the same way we already proved for
\er{EBDHTrans444knbuihuigGGjj11ghthgtjkjj11int} and
\er{EBDHTrans444knbuihuigGGjj11jgggggjj11hhhjjj11int} or
\er{EBDHTrans444knbuihuigGGjj11jgggggjj11hhhjjj11gvghhghgjj11int} we
can also prove that the above alternative
\er{EBDHTrans444knbuihuigGGjj11ghthgtjkjj11jj11int} and
\er{EBDHTrans444knbuihuigGGjj11jgggggjj11hhhjjj11jj11int} or
\er{EBDHTrans444knbuihuigGGjj11jgggggjj11hhhjjj11gvghhghgjj11jj11int}
is invariant under the change of inertial or non-inertial cartesian
coordinate system. Furthermore, as before, the above alternative is
equivalent to
\er{EBDHTrans444knbuihuigGGjj11jgggggjj11hhhjjj11gvghhghgjj11jgjggj11hghffh11int}
in the particular case $\vec u\equiv 0$. Moreover, in the case of an
arbitrary moving transparent medium without dispersion the
alternative law is equivalent to \er{EBDHTrans444knbuihuigGGint}.
Finally, as before,
\er{EBDHTrans444knbuihuigGGjj11ghthgtjkjj11jj11int} and
\er{EBDHTrans444knbuihuigGGjj11jgggggjj11hhhjjj11jj11int} or
\er{EBDHTrans444knbuihuigGGjj11jgggggjj11hhhjjj11gvghhghgjj11jj11int}
is also independent of the initial instant of time $t_0$.

We are unable to see any clear advantage of the dispersion law in
the form of \er{EBDHTrans444knbuihuigGGjj11ghthgtjkjj11int} and
\er{EBDHTrans444knbuihuigGGjj11jgggggjj11hhhjjj11int} or
\er{EBDHTrans444knbuihuigGGjj11jgggggjj11hhhjjj11gvghhghgjj11int}
with respect to the law
\er{EBDHTrans444knbuihuigGGjj11ghthgtjkjj11jj11int} and
\er{EBDHTrans444knbuihuigGGjj11jgggggjj11hhhjjj11jj11int} or
\er{EBDHTrans444knbuihuigGGjj11jgggggjj11hhhjjj11gvghhghgjj11jj11int}.
We just note that the above two forms of the dispersion law coincide
in the case where our medium moves as a \underline{rigid} body.

See also subsection \ref{hugygyggygygyfgyhkjhjhgj} to the
generalization of the dispersion laws to the case of anisotropic
mediums.

\subsection{Some further consequences of Maxwell equations}
Again consider the system of Maxwell equations in the vacuum or in a
medium, without dispersion, having the form
\er{MaxMedFullGGffykkhjhhzzkk}:
\begin{equation}\label{MaxVacFullPPNffGGint}
\begin{cases}
curl_{\vec x} \vec H=\frac{4\pi}{c}\vec j+
\frac{1}{c}\frac{\partial \vec D}{\partial t},\\
div_{\vec x} \vec D=4\pi\rho,\\
curl_{\vec x} \vec E+\frac{1}{c}\frac{\partial \vec B}{\partial t}=0,\\
div_{\vec x} \vec B=0\\
\vec E=\gamma_0\vec D-\frac{1}{c}\,\vec {\tilde u}\times \vec B\\
\vec H=\kappa_0\vec B+\frac{1}{c}\,\vec {\tilde u}\times \vec D,\\
\vec {\tilde u}=\left(\gamma_0\kappa_0\vec
v+(1-\gamma_0\kappa_0)\vec u\right),
\end{cases}
\end{equation}
%
%
%
%
%
%
where $\gamma_0\neq 0$ and $\kappa_0\neq 0$ are material
coefficients, $\vec v$ is the vectorial gravitational potential
$\vec u$ is the medium velocity and $\vec {\tilde
u}=\left(\gamma_0\kappa_0\vec v+(1-\gamma_0\kappa_0)\vec u\right)$
is the speed-like vector field. Remind that in the case of the
vacuum we have $\gamma_0=\kappa_0=1$, $\vec {\tilde u}=\vec v$ and
equations \er{MaxVacFullPPNffGGint} are precise (in the frames of
our model). Otherwise, in the case $\gamma_0\kappa_0\neq 1$
equations \er{MaxVacFullPPNffGGint} are just an approximation that
is good enough for the case:
\begin{equation}\label{OprdddsimGGffyhjyhhtygrffgfzzjjjint}
\frac{|\kappa_0||1-\gamma_0\kappa_0|\cdot|\vec u-\vec v|^2}{c^2}\ll
1.
\end{equation}
Throughout this section we study equation \er{MaxVacFullPPNffGGint}
in domains where we assume that the coefficients $\gamma_0\neq 0$
and $\kappa_0\neq 0$ vary sufficiently slow on the place and time
and thus their spatial and temporal derivatives are negligible with
respect to other terms. Next again by the third and the fourth
equations in \er{MaxVacFullPPNffGGint} we can write
\begin{equation}\label{MaxVacFull1bjkgjhjhgjgjgkjfhjfdghghligioiuittrPPNggint}
\begin{cases}
\vec B\equiv curl_{\vec x} \vec A,\\
\vec E\equiv-\nabla_{\vec x}\Psi-\frac{1}{c}\frac{\partial\vec
A}{\partial t},
\end{cases}
\end{equation}
where $\Psi$ and $\vec A$ are the usual scalar and the vectorial
electromagnetic potentials. Then by
\er{MaxVacFull1bjkgjhjhgjgjgkjfhjfdghghligioiuittrPPNggint} and
\er{MaxVacFullPPNffGGint} we have
\begin{equation}\label{vhfffngghPPN333yuyuGGint}
\begin{cases}
\vec B= curl_{\vec x} \vec A\\
\vec E=-\nabla_{\vec x}\Psi-\frac{1}{c}\frac{\partial\vec
A}{\partial t}\\
 \vec D=-\frac{1}{\gamma_0}\nabla_{\vec
x}\Psi-\frac{1}{\gamma_0 c}\frac{\partial\vec A}{\partial t}+\frac{1}{c\gamma_0}\vec {\tilde u}\times curl_{\vec x}\vec A\\
\vec H= \kappa_0 \,curl_{\vec x} \vec A+\frac{1}{c}\,\vec {\tilde
u}\times\left(-\frac{1}{\gamma_0}\nabla_{\vec
x}\Psi-\frac{1}{\gamma_0 c}\frac{\partial\vec A}{\partial
t}+\frac{1}{\gamma_0 c}\vec {\tilde u}\times curl_{\vec x}\vec
A\right).
\end{cases}
\end{equation}
Next we remind the definition of the proper scalar electromagnetic
potential:
\begin{equation}\label{vhfffngghhjghhgPPNghghghutghffugghjhjkjjklggkkkint}
\Psi_0:=\Psi-\frac{1}{c}\vec A\cdot\vec v,
\end{equation}
and remind also that $\vec A$ is a proper vector field and $\Psi_0$
is a proper scalar field. Then in the case of the medium we also
define an additional scalar electromagnetic potential:
\begin{equation}\label{vhfffngghhjghhgPPNghghghutghffugghjhjkjjklggint}
\Psi_1:=\Psi-\frac{1}{c}\vec A\cdot\vec {\tilde u}.
\end{equation}
Then, since $\vec A$ is a proper vector field, we deduce that
$\Psi_1$ is also a proper scalar field. Moreover, in the case of the
vacuum or more generally in the case where $\gamma_0\kappa_0\approx
1$ we have $\Psi_1=\Psi_0$. Thus by
\er{vhfffngghhjghhgPPNghghghutghffugghjhjkjjklggint} we rewrite
\er{vhfffngghPPN333yuyuGGint} as:
\begin{equation}\label{vhfffngghPPNffGGint}
\begin{cases}
\vec B= curl_{\vec x} \vec A\\
\vec E=-\nabla_{\vec x}\Psi_1-\frac{1}{c}\frac{\partial\vec
A}{\partial t}-\frac{1}{c}\nabla_{\vec x}\left(\vec A\cdot\vec
{\tilde u}\right)
\\
 \vec D=-\frac{1}{\gamma_0}\nabla_{\vec
x}\Psi_1-\frac{1}{\gamma_0 c}\left(\frac{\partial\vec A}{\partial
t}-\vec {\tilde u}\times curl_{\vec x}\vec A+\nabla_{\vec
x}\left(\vec A\cdot\vec {\tilde u}\right)\right)
\\
\vec H= \kappa_0\,curl_{\vec x} \vec A-\frac{1}{c}\,\vec {\tilde
u}\times
\left(\frac{1}{\gamma_0}\nabla_{\vec x}\Psi_1+\frac{1}{\gamma_0
c}\left(\frac{\partial\vec A}{\partial t}-\vec {\tilde u}\times
curl_{\vec x}\vec A+\nabla_{\vec x}\left(\vec A\cdot\vec {\tilde
u}\right)\right)\right).
\end{cases}
\end{equation}
Then by inserting \er{vhfffngghPPNffGGint} into
\er{MaxVacFullPPNffGGint} straightforward calculations presented in
subsection \ref{gcCM} lead to the following equations:
\begin{multline}\label{MaxVacFullPPNmmmffGG111jjkkhkjhdfjosjfint}
\left(\frac{\partial}{\partial t}\left\{\frac{1}{c}\Div_{\vec x}\vec
A\right\}+\Div_{\vec x}\left\{\frac{1}{c}\left(\Div_{\vec x}\vec
A\right)\vec {\tilde u}\right\}\right)+\Delta_{\vec x}\Psi_1\\+
\Div_{\vec x}\left\{\frac{1}{c}\left(d_{\vec x}\vec {\tilde
u}+\left\{d_{\vec x}\vec {\tilde u}\right\}^T\right)\cdot\vec
A-\frac{1}{c}\left(\Div_{\vec x}\vec {\tilde u}\right)\vec A\right\}
=-4\pi\gamma_0\rho
\end{multline}
and
\begin{multline}\label{MaxVacFullPPNnnnffGG111ikiojiojiscjkhhint}
-\Delta_{\vec x}\vec A =\frac{4\pi}{\kappa_0 c}\left(\vec j-\rho\vec
{\tilde u}\right)- \nabla_{\vec x}\left(\frac{1}{\kappa_0\gamma_0
c}\left(\frac{\partial\Psi_1}{\partial t}+\vec {\tilde
u}\cdot\nabla_{\vec x}\Psi_1\right)+\left(\Div_{\vec x}\vec
A\right)\right)
\\+ \left(\left(d_{\vec x}\vec {\tilde
u}+\left\{d_{\vec x}\vec {\tilde
u}\right\}^T\right)\cdot\left\{\frac{1}{\kappa_0\gamma_0
c}\nabla_{\vec x}\Psi_1\right\}-\left(\Div_{\vec x}\vec {\tilde
u}\right)\left\{\frac{1}{\kappa_0\gamma_0 c}\nabla_{\vec
x}\Psi_1\right\}\right)
\\-\frac{\partial}{\partial
t}\left\{\frac{1}{\gamma_0\kappa_0 c^2}\left(\frac{\partial\vec
A}{\partial t}+d_{\vec x}\vec A\cdot\vec {\tilde u}+\left\{d_{\vec
x}\vec {\tilde u}\right\}^T\cdot\vec A\right)\right\}\\- \Div_{\vec
x} \left\{\left(\frac{1}{\gamma_0\kappa_0
c^2}\left(\frac{\partial\vec A}{\partial t}+d_{\vec x}\vec
A\cdot\vec {\tilde u}+\left\{d_{\vec x}\vec
{\tilde u}\right\}^T\cdot\vec A\right)\right)\otimes\vec {\tilde u}\right\}\\
 +\left(d_{\vec
x}\vec {\tilde u}\cdot\left\{\frac{1}{\gamma_0\kappa_0
c^2}\left(\frac{\partial\vec A}{\partial t}+d_{\vec x}\vec
A\cdot\vec {\tilde u}+\left\{d_{\vec x}\vec {\tilde
u}\right\}^T\cdot\vec A\right)\right\}\right).
\end{multline}
Alternatively, we rewrite
\er{MaxVacFullPPNnnnffGG111ikiojiojiscjkhhint} as:
\begin{multline}\label{MaxVacFullPPNnnnffGG111ikiojiojiscjkhhopopint}
-\Delta_{\vec x}\vec A = \frac{4\pi}{\kappa_0 c}\left(\vec
j-\rho\vec {\tilde u}\right)-\nabla_{\vec x}\left(\Div_{\vec x}\vec
A\right)\\- \frac{1}{\gamma_0\kappa_0
c}\left(\frac{\partial}{\partial t}\left\{\nabla_{\vec
x}\Psi_1\right\}-curl_{\vec x}\left\{\vec {\tilde u}\times
\nabla_{\vec x}\Psi_1\right\}+\left(\Delta_{\vec x}\Psi_1\right)\vec
{\tilde u}\right)
\\-\frac{\partial}{\partial
t}\left\{\frac{1}{\gamma_0\kappa_0 c^2}\left(\frac{\partial\vec
A}{\partial t}+d_{\vec x}\vec A\cdot\vec {\tilde u}+\left\{d_{\vec
x}\vec {\tilde u}\right\}^T\cdot\vec A\right)\right\}\\- \Div_{\vec
x} \left\{\left(\frac{1}{\gamma_0\kappa_0
c^2}\left(\frac{\partial\vec A}{\partial t}+d_{\vec x}\vec
A\cdot\vec {\tilde u}+\left\{d_{\vec x}\vec
{\tilde u}\right\}^T\cdot\vec A\right)\right)\otimes\vec {\tilde u}\right\}\\
 +\left(d_{\vec
x}\vec {\tilde u}\cdot\left\{\frac{1}{\gamma_0\kappa_0
c^2}\left(\frac{\partial\vec A}{\partial t}+d_{\vec x}\vec
A\cdot\vec {\tilde u}+\left\{d_{\vec x}\vec {\tilde
u}\right\}^T\cdot\vec A\right)\right\}\right).
\end{multline}

Next if we assume the following calibration of the potentials:
\begin{equation}\label{MaxVacFullPPNjjjjffhhGGGGGGint}
\Div_{\vec x}\vec A=0,
\end{equation}
then by \er{MaxVacFullPPNjjjjffhhGGGGGGint},
\er{MaxVacFullPPNmmmffGG111jjkkhkjhdfjosjfint} and
\er{MaxVacFullPPNnnnffGG111ikiojiojiscjkhhopopint} we have
\begin{equation}\label{MaxVacFullPPNmmmffffffhhtygghGGGGint}
-\Delta_{\vec x}\Psi_1=4\pi\gamma_0\rho+\Div_{\vec x}
\left\{\frac{1}{c}\left(d_{\vec x}\vec {\tilde u}+\left\{d_{\vec
x}\vec {\tilde u}\right\}^T\right)\cdot\vec
A-\frac{1}{c}\left(\Div_{\vec x}\vec {\tilde u}\right)\vec
A\right\},
\end{equation}
and
\begin{multline}\label{MaxVacFullPPNnnnffffffyuughjhjhjhhjjkjhkkjhhjhghGGGGint}
-\Delta_{\vec x}\vec A = \frac{4\pi}{\kappa_0 c}\left(\vec j-
\frac{1}{4\pi\gamma_0}\left(\frac{\partial}{\partial
t}\left\{\nabla_{\vec x}\Psi_1\right\}-curl_{\vec x}\left\{\vec
{\tilde u}\times \nabla_{\vec
x}\Psi_1\right\}\right)\right)\\+\frac{1}{\gamma_0\kappa_0
c^2}\left( \Div_{\vec x} \left\{\left(d_{\vec x}\vec {\tilde
u}+\left\{d_{\vec x}\vec {\tilde u}\right\}^T\right)\cdot\vec
A-\left(\Div_{\vec x}\vec {\tilde u}\right)\vec A\right\}\right)\vec
{\tilde u}
\\-\frac{\partial}{\partial
t}\left\{\frac{1}{\gamma_0\kappa_0 c^2}\left(\frac{\partial\vec
A}{\partial t}+d_{\vec x}\vec A\cdot\vec {\tilde u}+\left\{d_{\vec
x}\vec {\tilde u}\right\}^T\cdot\vec A\right)\right\}\\- \Div_{\vec
x} \left\{\left(\frac{1}{\gamma_0\kappa_0
c^2}\left(\frac{\partial\vec A}{\partial t}+d_{\vec x}\vec
A\cdot\vec {\tilde u}+\left\{d_{\vec x}\vec
{\tilde u}\right\}^T\cdot\vec A\right)\right)\otimes\vec {\tilde u}\right\}\\
 +\left(d_{\vec
x}\vec {\tilde u}\cdot\left\{\frac{1}{\gamma_0\kappa_0
c^2}\left(\frac{\partial\vec A}{\partial t}+d_{\vec x}\vec
A\cdot\vec {\tilde u}+\left\{d_{\vec x}\vec {\tilde
u}\right\}^T\cdot\vec A\right)\right\}\right),
\end{multline}
(See subsection \ref{gcCM} for details).

On the other hand, if we assume the following alternative
calibration of the potentials:
\begin{equation}\label{MaxVacFullPPNjjjjffhhGGint}
\frac{1}{\kappa_0\gamma_0 c}\left(\frac{\partial\Psi_1}{\partial
t}+\vec {\tilde u}\cdot\nabla_{\vec x}\Psi_1\right)+\Div_{\vec
x}\vec A=0,
\end{equation}
then by \er{MaxVacFullPPNmmmffGG111jjkkhkjhdfjosjfint},
\er{MaxVacFullPPNnnnffGG111ikiojiojiscjkhhint} we have
\begin{multline}\label{MaxVacFullPPNmmmffffffiuiuhjuGGint}
\left(\frac{\partial}{\partial t}\left\{\frac{1}{\kappa_0\gamma_0
c^2}\left(\frac{\partial\Psi_1}{\partial t}+\vec {\tilde
u}\cdot\nabla_{\vec x}\Psi_1\right)\right\}+\Div_{\vec x}
\left\{\frac{1}{\kappa_0\gamma_0
c^2}\left(\frac{\partial\Psi_1}{\partial t}+\vec {\tilde
u}\cdot\nabla_{\vec x}\Psi_1\right)\vec {\tilde
u}\right\}\right)-\Delta_{\vec
x}\Psi_1\\=4\pi\gamma_0\rho+\Div_{\vec x}
\left\{\frac{1}{c}\left(d_{\vec x}\vec {\tilde u}+\left\{d_{\vec
x}\vec {\tilde u}\right\}^T\right)\cdot\vec
A-\frac{1}{c}\left(\Div_{\vec x}\vec {\tilde u}\right)\vec
A\right\},
\end{multline}
and
\begin{multline}\label{MaxVacFullPPNnnnffffffyuughjhjhjhhjjkjhkkjhujgGGint}
-\Delta_{\vec x}\vec A= \frac{4\pi}{\kappa_0 c}\left(\vec j-\rho\vec
{\tilde u}\right)+ \left(\left(d_{\vec x}\vec {\tilde
u}+\left\{d_{\vec x}\vec {\tilde
u}\right\}^T\right)\cdot\left\{\frac{1}{\kappa_0\gamma_0
c}\nabla_{\vec x}\Psi_1\right\}-\left(\Div_{\vec x}\vec {\tilde
u}\right)\left\{\frac{1}{\kappa_0\gamma_0 c}\nabla_{\vec
x}\Psi_1\right\}\right)
\\-\frac{\partial}{\partial
t}\left\{\frac{1}{\gamma_0\kappa_0 c^2}\left(\frac{\partial\vec
A}{\partial t}+d_{\vec x}\vec A\cdot\vec {\tilde u}+\left\{d_{\vec
x}\vec {\tilde u}\right\}^T\cdot\vec A\right)\right\}\\- \Div_{\vec
x} \left\{\left(\frac{1}{\gamma_0\kappa_0
c^2}\left(\frac{\partial\vec A}{\partial t}+d_{\vec x}\vec
A\cdot\vec {\tilde u}+\left\{d_{\vec x}\vec
{\tilde u}\right\}^T\cdot\vec A\right)\right)\otimes\vec {\tilde u}\right\}\\
 +\left(d_{\vec
x}\vec {\tilde u}\cdot\left\{\frac{1}{\gamma_0\kappa_0
c^2}\left(\frac{\partial\vec A}{\partial t}+d_{\vec x}\vec
A\cdot\vec {\tilde u}+\left\{d_{\vec x}\vec {\tilde
u}\right\}^T\cdot\vec A\right)\right\}\right).
\end{multline}

Next, from now we assume that $\vec {\tilde u}$ varies sufficiently
slowly in space and time variables, so that the following
approximation is valid:
\begin{equation}\label{MaxVacFullPPNmmmffffffhhtygghGGGGyuhggghghghffgiooiiojhjjhhjuiuiugughhjggjjhkoououint}
\frac{\left| d_{\vec x}\vec {\tilde u}+\left\{d_{\vec x}\vec {\tilde
u}\right\}^T\right|}{ c}\ll \frac{|\nabla_{\vec x}\Psi_1|}{| \vec
A|+|\Psi_1|}+\frac{|\gamma_0\kappa_0||\Delta_{\vec x}\vec A|}{|
d_{\vec x}\vec A|+|\nabla_{\vec x}\Psi_1|}.
\end{equation}
In particular, if in some Cartesian coordinate system we have both
\begin{equation}\label{MaxVacFullPPNmmmffffffhhtygghGGGGyuhggghghghffgiooiiojhjjhhjuiuiugughhjggjjhkoououioioint}
\frac{\left| d_{\vec x}\vec {\tilde u}+\left\{d_{\vec x}\vec {\tilde
u}\right\}^T\right|^2}{
\left|\vec {\tilde u}\right|^2}\ll \left(\frac{|\nabla_{\vec
x}\Psi_1|}{| \vec
A|+|\Psi_1|}\right)^2+\left(\frac{|\gamma_0\kappa_0||\Delta_{\vec
x}\vec A|}{| d_{\vec x}\vec A|+|\nabla_{\vec x}\Psi_1|}\right)^2,
\end{equation}
and
\begin{equation}\label{MaxVacFullPPNmmmffffffiuiuhjuughbghhiuijghghghhjhjhhghyuyhhjhjihikjjjjkjljklghhggghiioioint}
\frac{\left|\vec {\tilde u}\right|^2}{
c^2}\ll 1,
\end{equation}
then
\er{MaxVacFullPPNmmmffffffhhtygghGGGGyuhggghghghffgiooiiojhjjhhjuiuiugughhjggjjhkoououint}
indeed holds! Furthermore, taking into the account
\er{MaxVacFullPPNmmmffffffhhtygghGGGGyuhggghghghffgiooiiojhjjhhjuiuiugughhjggjjhkoououint},
under the calibration \er{MaxVacFullPPNjjjjffhhGGGGGGint}, we
rewrite \er{MaxVacFullPPNmmmffffffhhtygghGGGGint} and
\er{MaxVacFullPPNnnnffffffyuughjhjhjhhjjkjhkkjhhjhghGGGGint} as
\begin{equation}\label{MaxVacFullPPNmmmffffffhhtygghGGGGggint}
-\Delta_{\vec x}\Psi_1\approx 4\pi\gamma_0\rho,
\end{equation}
and
\begin{multline}\label{MaxVacFullPPNnnnffffffyuughjhjhjhhjjkjhkkjhhjhghGGGGggint}
-\Delta_{\vec x}\vec A\approx \frac{4\pi}{\kappa_0 c}\left(\vec j-
\frac{1}{4\pi\gamma_0}\left(\frac{\partial}{\partial
t}\left\{\nabla_{\vec x}\Psi_1\right\}-curl_{\vec x}\left\{\vec
{\tilde u}\times \nabla_{\vec x}\Psi_1\right\}\right)\right)
\\-\frac{\partial}{\partial
t}\left\{\frac{1}{\gamma_0\kappa_0 c^2}\left(\frac{\partial\vec
A}{\partial t}+d_{\vec x}\vec A\cdot\vec {\tilde u}+\left\{d_{\vec
x}\vec {\tilde u}\right\}^T\cdot\vec A\right)\right\}\\- \Div_{\vec
x} \left\{\left(\frac{1}{\gamma_0\kappa_0
c^2}\left(\frac{\partial\vec A}{\partial t}+d_{\vec x}\vec
A\cdot\vec {\tilde u}+\left\{d_{\vec x}\vec
{\tilde u}\right\}^T\cdot\vec A\right)\right)\otimes\vec {\tilde u}\right\}\\
 +\left(d_{\vec
x}\vec {\tilde u}\cdot\left\{\frac{1}{\gamma_0\kappa_0
c^2}\left(\frac{\partial\vec A}{\partial t}+d_{\vec x}\vec
A\cdot\vec {\tilde u}+\left\{d_{\vec x}\vec {\tilde
u}\right\}^T\cdot\vec A\right)\right\}\right).
\end{multline}
Note that, using Proposition \ref{yghgjtgyrtrtint} we deduce that
the approximate equations
\er{MaxVacFullPPNmmmffffffhhtygghGGGGggint} and
\er{MaxVacFullPPNnnnffffffyuughjhjhjhhjjkjhkkjhhjhghGGGGggint} are
still invariant under the change of inertial or non-inertial
cartesian coordinate system, provided that $\vec A$ is a proper
vector field and $\Psi_1$ is a proper scalar field.

On the other hand, taking into the account
\er{MaxVacFullPPNmmmffffffhhtygghGGGGyuhggghghghffgiooiiojhjjhhjuiuiugughhjggjjhkoououint},
under the calibration \er{MaxVacFullPPNjjjjffhhGGint}, we rewrite
\er{MaxVacFullPPNmmmffffffiuiuhjuGGint} and
\er{MaxVacFullPPNnnnffffffyuughjhjhjhhjjkjhkkjhujgGGint} as
\begin{equation}\label{MaxVacFullPPNmmmffffffiuiuhjuGGggint}
\frac{1}{\kappa_0\gamma_0 c^2}\left(\frac{\partial}{\partial
t}\left(\frac{\partial\Psi_1}{\partial t}+\vec {\tilde
u}\cdot\nabla_{\vec x}\Psi_1\right)+div_{\vec x}
\left\{\left(\frac{\partial\Psi_1}{\partial t}+\vec {\tilde
u}\cdot\nabla_{\vec x}\Psi_1\right)\vec {\tilde
u}\right\}\right)-\Delta_{\vec x}\Psi_1\approx 4\pi\gamma_0\rho.
\end{equation}
and
\begin{multline}\label{MaxVacFullPPNnnnffffffyuughjhjhjhhjjkjhkkjhujgGGggint}
-\Delta_{\vec x}\vec A\approx \frac{4\pi}{\kappa_0 c}\left(\vec
j-\rho\vec {\tilde u}\right)-\frac{\partial}{\partial
t}\left\{\frac{1}{\gamma_0\kappa_0 c^2}\left(\frac{\partial\vec
A}{\partial t}+d_{\vec x}\vec A\cdot\vec {\tilde u}+\left\{d_{\vec
x}\vec {\tilde u}\right\}^T\cdot\vec A\right)\right\}\\- \Div_{\vec
x} \left\{\left(\frac{1}{\gamma_0\kappa_0
c^2}\left(\frac{\partial\vec A}{\partial t}+d_{\vec x}\vec
A\cdot\vec {\tilde u}+\left\{d_{\vec x}\vec
{\tilde u}\right\}^T\cdot\vec A\right)\right)\otimes\vec {\tilde u}\right\}\\
 +\left(d_{\vec
x}\vec {\tilde u}\cdot\left\{\frac{1}{\gamma_0\kappa_0
c^2}\left(\frac{\partial\vec A}{\partial t}+d_{\vec x}\vec
A\cdot\vec {\tilde u}+\left\{d_{\vec x}\vec {\tilde
u}\right\}^T\cdot\vec A\right)\right\}\right).
\end{multline}
Again note that, using Proposition \ref{yghgjtgyrtrtint} we deduce
that the approximate equations
\er{MaxVacFullPPNmmmffffffiuiuhjuGGggint} and
\er{MaxVacFullPPNnnnffffffyuughjhjhjhhjjkjhkkjhujgGGggint} are still
invariant under the change of inertial or non-inertial cartesian
coordinate system, provided that $\vec A$ is a proper vector field
and $\Psi_1$ is a proper scalar field.

In particular, assume that in some Cartesian coordinate system $(*)$
we have the following stronger than
\er{MaxVacFullPPNmmmffffffhhtygghGGGGyuhggghghghffgiooiiojhjjhhjuiuiugughhjggjjhkoououint}
approximation:
\begin{equation}\label{MaxVacFullPPNmmmffffffhhtygghGGGGyuhggghghint}
\frac{\left| d_{\vec x}\vec {\tilde u}\right|}{
c}\ll \frac{|\gamma_0\kappa_0|| \Delta_{\vec x}\vec A|
+\frac{1}{c}\left|d_{\vec x}\left\{\frac{\partial \vec A}{\partial
t}\right\}\right|+\frac{1}{c^2}\left|\frac{\partial^2 \vec
A}{\partial t^2}\right|+|\gamma_0\kappa_0|| \nabla^2_{\vec
x}\Psi_1|+\frac{1}{c}\left|\frac{\partial}{\partial
t}\left(\nabla_{\vec
x}\Psi_1\right)\right|+\frac{1}{c^2}\left|\frac{\partial^2
\Psi_1}{\partial t^2}\right|}{| d_{\vec x}\vec
A|+\frac{1}{c}\left|\frac{\partial \vec A}{\partial
t}\right|+|\nabla_{\vec x}\Psi_1|+\frac{1}{c}\left|\frac{\partial
\Psi_1}{\partial t}\right|},
\end{equation}
i.e. the field $\vec {\tilde u}$ changes in space much slower than
$(\Psi_1,\vec A)$. Estimation
\er{MaxVacFullPPNmmmffffffhhtygghGGGGyuhggghghint} holds especially
good for the electromagnetic waves of high frequency for example for
the visible light. However,
\er{MaxVacFullPPNmmmffffffhhtygghGGGGyuhggghghint} is still well for
almost every electromagnetic field we meet in the common life,
except probably the magnetic field of the Earth. Then, by
\er{MaxVacFullPPNmmmffffffiuiuhjuGGggint},
\er{MaxVacFullPPNnnnffffffyuughjhjhjhhjjkjhkkjhujgGGggint} and
\er{MaxVacFullPPNmmmffffffhhtygghGGGGyuhggghghint} we can write the
further approximating equations:
\begin{equation}\label{MaxVacFullPPNmmmffffffiuiuhjuGGggFGint}
\left(\frac{\partial}{\partial
t}\left\{\frac{1}{c^2_0}\left(\frac{\partial\Psi_1}{\partial t}+\vec
{\tilde u}\cdot\nabla_{\vec x}\Psi_1\right)\right\}+\Div_{\vec x}
\left\{\frac{1}{c^2_0}\left(\frac{\partial\Psi_1}{\partial t}+\vec
{\tilde u}\cdot\nabla_{\vec x}\Psi_1\right)\vec {\tilde
u}\right\}\right)-\Delta_{\vec x}\Psi_1 \approx 4\pi\gamma_0\rho,
\end{equation}
and
\begin{equation}\label{MaxVacFullPPNnnnffffffyuughjhjhjhhjjkjhkkjhujgGGggFGint}
\left(\frac{\partial}{\partial
t}\left\{\frac{1}{c^2_0}\left(\frac{\partial\vec A}{\partial
t}+d_{\vec x}\vec A\cdot\vec {\tilde u}\right)\right\}+\Div_{\vec x}
\left\{\frac{1}{c^2_0}\left(\frac{\partial\vec A}{\partial
t}+d_{\vec x}\vec A\cdot\vec {\tilde u}\right)\otimes\vec {\tilde
u}\right\}\right)-\Delta_{\vec x}\vec A\approx \frac{4\pi}{\kappa_0
c}\left(\vec j-\rho\vec {\tilde u}\right),
\end{equation}
where the scalar quantity $c_0$, defined by
\begin{equation}\label{gughhghfbvnbvint}
c_0=c\sqrt{\kappa_0\gamma_0},
\end{equation}
is called speed of light in the medium. Note that, although the
approximate equations \er{MaxVacFullPPNmmmffffffiuiuhjuGGggFGint}
and \er{MaxVacFullPPNnnnffffffyuughjhjhjhhjjkjhkkjhujgGGggFGint} are
invariant under the Galilean Transformation, they are not invariant
under the more general change of non-inertial cartesian coordinate
system. However, \er{MaxVacFullPPNmmmffffffiuiuhjuGGggFGint} and
\er{MaxVacFullPPNnnnffffffyuughjhjhjhhjjkjhkkjhujgGGggFGint} are
more convenient than \er{MaxVacFullPPNmmmffffffiuiuhjuGGggint} and
\er{MaxVacFullPPNnnnffffffyuughjhjhjhhjjkjhkkjhujgGGggint}, since
the scalar potential $\Psi_1$ and every of the three scalar
components of the vector potential $\vec A$ in
\er{MaxVacFullPPNmmmffffffiuiuhjuGGggFGint} and
\er{MaxVacFullPPNnnnffffffyuughjhjhjhhjjkjhkkjhujgGGggFGint}
satisfies four decoupled equations of the same type, that differ
only by the right parts. On the other hand, if we consider some
three proper vector fields $\vec e_1:=\vec e_1(\vec x,t)$, $\vec
e_2:=\vec e_2(\vec x,t)$, and $\vec e_3:=\vec e_3(\vec x,t)$, which
are mutually orthogonal to each other and satisfy the following
approximation:
\begin{equation}\label{MaxVacFullPPNmmmffffffhhtygghGGGGyuhggghghhjhjhint}
\frac{|d_{\vec x}\vec {e}_1|+c_0|\partial_t\vec e_1|}{|\vec
{e}_1|}+\frac{|d_{\vec x}\vec {e}_2|+c_0|\partial_t\vec e_2|}{|\vec
{e}_2|}+\frac{|d_{\vec x}\vec {e}_3|+c_0|\partial_t\vec e_3|}{|\vec
{e}_3|}\ll\frac{|\nabla_{\vec
x}\Psi_1|+c_0|\partial_t\Psi_1|+|d_{\vec x}\vec
A|+c_0|\partial_t\vec A|}{|\Psi_1|+|\vec A|}.
\end{equation}
in the given before coordinate system $(*)$, i.e. the field $\vec
{e}_k$ vary in space and time much weaker than $(\Psi_1,\vec A)$,
then we may write the alternative to
\er{MaxVacFullPPNnnnffffffyuughjhjhjhhjjkjhkkjhujgGGggFGint} and
\er{MaxVacFullPPNmmmffffffiuiuhjuGGggFGint} approximate equations in
the form of four decoupled scalar invariant wave equations of the
same type:
\begin{multline}\label{MaxVacFullPPNmmmffffffiuiuhjuGGggFGjjkkjjint}
\left(\frac{\partial}{\partial
t}\left\{\frac{1}{c^2_0}\left(\frac{\partial}{\partial t}\left(\vec
e_k\cdot\vec A\right)+\vec {\tilde u}\cdot\nabla_{\vec x}\left(\vec
e_k\cdot\vec A\right)\right)\right\}+\Div_{\vec x}
\left\{\frac{1}{c^2_0}\left(\frac{\partial}{\partial t}\left(\vec
e_k\cdot\vec A\right)+\vec {\tilde u}\cdot\nabla_{\vec x}\left(\vec
e_k\cdot\vec A\right)\right)\vec {\tilde
u}\right\}\right)\\-\Delta_{\vec x}\left(\vec e_k\cdot\vec A\right)
\,\approx\, \frac{4\pi}{\kappa_0 c}\left(\vec j-\rho\vec {\tilde
u}\right)\cdot\vec e_k\quad\quad\forall k=1,2,3,
\end{multline}
and
\begin{equation}\label{MaxVacFullPPNmmmffffffiuiuhjuGGggFGkjkljlint}
\left(\frac{\partial}{\partial
t}\left\{\frac{1}{c^2_0}\left(\frac{\partial\Psi_1}{\partial t}+\vec
{\tilde u}\cdot\nabla_{\vec x}\Psi_1\right)\right\}+\Div_{\vec x}
\left\{\frac{1}{c^2_0}\left(\frac{\partial\Psi_1}{\partial t}+\vec
{\tilde u}\cdot\nabla_{\vec x}\Psi_1\right)\vec {\tilde
u}\right\}\right)-\Delta_{\vec x}\Psi_1 \approx 4\pi\gamma_0\rho.
\end{equation}
Then, clearly, the new alternative approximate equations
\er{MaxVacFullPPNmmmffffffiuiuhjuGGggFGjjkkjjint},
\er{MaxVacFullPPNmmmffffffiuiuhjuGGggFGkjkljlint} are indeed
invariant under the more general change of non-inertial cartesian
coordinate system. So we can use approximate equations
\er{MaxVacFullPPNmmmffffffiuiuhjuGGggFGjjkkjjint} and
\er{MaxVacFullPPNmmmffffffiuiuhjuGGggFGkjkljlint} in the arbitrary
Cartesian coordinate system $(*)$ even if
\er{MaxVacFullPPNmmmffffffhhtygghGGGGyuhggghghint} and
\er{MaxVacFullPPNmmmffffffhhtygghGGGGyuhggghghhjhjhint} are not
satisfied in the system $(*)$, provided that
\er{MaxVacFullPPNmmmffffffhhtygghGGGGyuhggghghint} and
\er{MaxVacFullPPNmmmffffffhhtygghGGGGyuhggghghhjhjhint} are
satisfied in another Cartesian system $(**)$.

In the absence of charges and currents (for example for
electromagnetic waves) equations
\er{MaxVacFullPPNmmmffffffiuiuhjuGGggFGint} and
\er{MaxVacFullPPNnnnffffffyuughjhjhjhhjjkjhkkjhujgGGggFGint} become:
\begin{equation}\label{MaxVacFullPPNmmmffffffiuiuhjuGGggFGelint}
\left(\frac{\partial}{\partial
t}\left\{\frac{1}{c^2_0}\left(\frac{\partial\Psi_1}{\partial t}+\vec
{\tilde u}\cdot\nabla_{\vec x}\Psi_1\right)\right\}+\Div_{\vec x}
\left\{\frac{1}{c^2_0}\left(\frac{\partial\Psi_1}{\partial t}+\vec
{\tilde u}\cdot\nabla_{\vec x}\Psi_1\right)\vec {\tilde
u}\right\}\right)-\Delta_{\vec x}\Psi_1=0,
\end{equation}
and
\begin{equation}\label{MaxVacFullPPNnnnffffffyuughjhjhjhhjjkjhkkjhujgGGggFGelint}
\left(\left\{\frac{1}{c^2_0}\frac{\partial}{\partial
t}\left(\frac{\partial\vec A}{\partial t}+d_{\vec x}\vec A\cdot\vec
{\tilde u}\right)\right\}+\Div_{\vec x}
\left\{\frac{1}{c^2_0}\left(\frac{\partial\vec A}{\partial
t}+d_{\vec x}\vec A\cdot\vec {\tilde u}\right)\otimes\vec {\tilde
u}\right\}\right)-\Delta_{\vec x}\vec A=0,
\end{equation}
and equations \er{MaxVacFullPPNmmmffffffiuiuhjuGGggFGjjkkjjint},
\er{MaxVacFullPPNmmmffffffiuiuhjuGGggFGkjkljlint} become:
\begin{multline}\label{MaxVacFullPPNmmmffffffiuiuhjuGGggFGjjkkjjkojjkint}
\left(\frac{\partial}{\partial
t}\left\{\frac{1}{c^2_0}\left(\frac{\partial}{\partial t}\left(\vec
e_k\cdot\vec A\right)+\vec {\tilde u}\cdot\nabla_{\vec x}\left(\vec
e_k\cdot\vec A\right)\right)\right\}+\Div_{\vec x}
\left\{\frac{1}{c^2_0}\left(\frac{\partial}{\partial t}\left(\vec
e_k\cdot\vec A\right)+\vec {\tilde u}\cdot\nabla_{\vec x}\left(\vec
e_k\cdot\vec A\right)\right)\vec {\tilde
u}\right\}\right)\\-\Delta_{\vec x}\left(\vec e_k\cdot\vec A\right)
\,\approx\, 0\quad\quad\forall k=1,2,3,
\end{multline}
and
\begin{equation}\label{MaxVacFullPPNmmmffffffiuiuhjuGGggFGkjkljlhint}
\left(\frac{\partial}{\partial
t}\left\{\frac{1}{c^2_0}\left(\frac{\partial\Psi_1}{\partial t}+\vec
{\tilde u}\cdot\nabla_{\vec x}\Psi_1\right)\right\}+\Div_{\vec x}
\left\{\frac{1}{c^2_0}\left(\frac{\partial\Psi_1}{\partial t}+\vec
{\tilde u}\cdot\nabla_{\vec x}\Psi_1\right)\vec {\tilde
u}\right\}\right)-\Delta_{\vec x}\Psi_1 \approx 0.
\end{equation}
Therefore, by \er{vhfffngghPPNffGGint}, differentiating
\er{MaxVacFullPPNmmmffffffiuiuhjuGGggFGelint} and
\er{MaxVacFullPPNnnnffffffyuughjhjhjhhjjkjhkkjhujgGGggFGelint} or
\er{MaxVacFullPPNmmmffffffiuiuhjuGGggFGjjkkjjkojjkint} and
\er{MaxVacFullPPNmmmffffffiuiuhjuGGggFGkjkljlhint} and further usage
of \er{MaxVacFullPPNmmmffffffhhtygghGGGGyuhggghghint} and
\er{MaxVacFullPPNmmmffffffhhtygghGGGGyuhggghghhjhjhint} gives that
if the scalar field $U:=U(\vec x,t)$ is one of any three scalar
components of every of the fields $\vec E$, $\vec B$, $\vec D$ or
$\vec H$, then $U$ satisfies the following approximate scalar
equation of the wave type:
\begin{equation}\label{MaxVacFullPPNmmmffffffiuiuhjuGGggFGelGHGHGHGGint}
\left(\frac{\partial}{\partial
t}\left\{\frac{1}{c^2_0}\left(\frac{\partial U}{\partial t}+\vec
{\tilde u}\cdot\nabla_{\vec x}U\right)\right\}+\Div_{\vec x}
\left\{\frac{1}{c^2_0}\left(\frac{\partial U}{\partial t}+\vec
{\tilde u}\cdot\nabla_{\vec x} U\right)\vec {\tilde
u}\right\}\right)-\Delta_{\vec x}U \approx 0,
\end{equation}
where,
\begin{equation}\label{uyuyuyyint}
\vec {\tilde u}=\left(\gamma_0\kappa_0\vec
v+\left(1-\gamma_0\kappa_0\right)\vec u\right)\,.
\end{equation}
Note that the wave equation
\er{MaxVacFullPPNmmmffffffiuiuhjuGGggFGelGHGHGHGGint} is of the same
type as
\er{MaxVacFull1ninshtrgravortghhghgjkgghklhjgkghghjjkjhjkkggjkhjkhjjhhfhjhkjhjjhkhbbgjlmkmmhzzzyyykkknnnNWBWHWPPNjjlllkkkkkkljkkjjkjkhjjhhgghgggghjjhgjgghint},
\er{MaxVacFull1ninshtrgravortghhghgjkgghklhjgkghghjjkjhjkkggjkhjkhjjhhfhjhkjhjjhkhbbgjlmkmmhzzzyyykkknnnNWBWHWPPNjjlllkkkkkkljkkjjkjkhjjhhgghgggghjjhgjgghjkjkjkint},
\er{MaxVacFull1ninshtrgravortghhghgjkgghklhjgkghghjjkjhjkkggjkhjkhjjhhfhjhkjkhbbgjhzzzyyykkknnnNWBWHWPPNjjjljkjkkhhhhhhhhhhhhfnfffhuyuyuyutjhjhjhggjjfgfhfhjhhgjgjhkjhjhkhkhjjint}
or \er{MaxVacFullPPNmmmffffffiuiuhjuGGggFGjjkkjjffjjint}, with the
only difference that the field $\vec {\tilde u}$ appear in
\er{MaxVacFullPPNmmmffffffiuiuhjuGGggFGelGHGHGHGGint} instead of the
field $\vec u_0$ in
\er{MaxVacFull1ninshtrgravortghhghgjkgghklhjgkghghjjkjhjkkggjkhjkhjjhhfhjhkjhjjhkhbbgjlmkmmhzzzyyykkknnnNWBWHWPPNjjlllkkkkkkljkkjjkjkhjjhhgghgggghjjhgjgghint},
\er{MaxVacFull1ninshtrgravortghhghgjkgghklhjgkghghjjkjhjkkggjkhjkhjjhhfhjhkjhjjhkhbbgjlmkmmhzzzyyykkknnnNWBWHWPPNjjlllkkkkkkljkkjjkjkhjjhhgghgggghjjhgjgghjkjkjkint}.
\er{MaxVacFull1ninshtrgravortghhghgjkgghklhjgkghghjjkjhjkkggjkhjkhjjhhfhjhkjkhbbgjhzzzyyykkknnnNWBWHWPPNjjjljkjkkhhhhhhhhhhhhfnfffhuyuyuyutjhjhjhggjjfgfhfhjhhgjgjhkjhjhkhkhjjint}
or \er{MaxVacFullPPNmmmffffffiuiuhjuGGggFGjjkkjjffjjint}.

\subsubsection{The case of quasistationary electromagnetic fields
inside a slowly moving medium in a weak gravitational
field}
Assume that in the given inertial or non-inertial
cartesian coordinate system $(*)$ the field $\vec {\tilde u}$ is
weak, meaning that at any instant on every point:
\begin{equation}\label{slowaetherGGaaint}
\frac{1}{\kappa_0\gamma_0}\frac{|\vec {\tilde u}|^2}{c^2}\,\ll\, 1.
\end{equation}
Here $\vec {\tilde u}=\left(\gamma_0\kappa_0\vec
v+(1-\gamma_0\kappa_0)\vec u\right)$ is the speed-like vector field,
where $\vec v$ is a vectorial gravitational potential in the system
$(*)$ and $\vec u$ is the medium velocity. Furthermore, consider
quasistationary electromagnetic fields. This means the following:
assume that the changes in time of the physical characteristics of
the electromagnetic fields become essential after certain interval
of time $T_e$ and the changes in space of the physical
characteristics of the fields become essential in the spatial
landscape $L_e$. Then we assume that
\begin{equation}\label{slochangGGaaint}
(\kappa_0\gamma_0)\frac{c^2T^2_e}{L^2_e}\,\gg\, 1.
\end{equation}
Next assume that we are under the calibration
\er{MaxVacFullPPNjjjjffhhGGGGGGint}. Then by \er{slowaetherGGaaint}
and \er{slochangGGaaint} we rewrite
\er{MaxVacFullPPNmmmffffffhhtygghGGGGint} and
\er{MaxVacFullPPNnnnffffffyuughjhjhjhhjjkjhkkjhhjhghGGGGint} as
\begin{equation}\label{MaxVacFullPPNmmmffffffhhtygghGGGGaaint}
-\Delta_{\vec x}\Psi_1=4\pi\gamma_0\rho+\frac{1}{c}\,div_{\vec x}
\left\{\left(d_{\vec x}\vec {\tilde u}+\left\{d_{\vec x}\vec {\tilde
u}\right\}^T\right)\cdot\vec A-\left(div_{\vec x}\vec {\tilde
u}\right)\vec A\right\},
\end{equation}
and
\begin{equation}\label{MaxVacFullPPNnnnffffffyuughjhjhjhhjjkjhkkjhhjhghGGGGaa1int}
-\Delta_{\vec x}\vec A\approx \frac{4\pi}{\kappa_0 c}\left(\vec
j-\rho\vec {\tilde u}\right)-\frac{1}{\kappa_0\gamma_0
c}\left(\frac{\partial}{\partial t}\left(\nabla_{\vec
x}\Psi_1\right)-curl_{\vec x}\left(\vec {\tilde u}\times\nabla_{\vec
x}\Psi_1\right)+\left(\Delta_{\vec x}\Psi_1\right)\vec {\tilde
u}\right).
\end{equation}
Moreover, by \er{slowaetherGGaaint} and \er{slochangGGaaint} we can
perform further approximation of
\er{MaxVacFullPPNnnnffffffyuughjhjhjhhjjkjhkkjhhjhghGGGGaa1int} and
we get
\begin{equation}\label{MaxVacFullPPNnnnffffffyuughjhjhjhhjjkjhkkjhhjhghGGGGaaint}
-\Delta_{\vec x}\vec A
\approx\frac{4\pi}{\kappa_0 c}\,\vec
j-\frac{1}{\kappa_0 c}\left(\frac{\partial}{\partial
t}\left(\nabla_{\vec x}\psi_0\right)-curl_{\vec x}\left(\vec
v\times\nabla_{\vec x}\psi_0\right)\right),
\end{equation}
where $\psi_0(\vec x,t)$ is the classical Coulomb's potential which
satisfies
\begin{equation}\label{columbPPNaaint}
-\Delta_{\vec x}\psi_0\equiv 4\pi\rho.
\end{equation}
So we rewrite \er{MaxVacFullPPNmmmffffffhhtygghGGGGaaint} and
\er{MaxVacFullPPNnnnffffffyuughjhjhjhhjjkjhkkjhhjhghGGGGaaint} as
\begin{equation}\label{MaxVacFull1bjkgjhjhgjgjgkjfhjfdghcgjhhjgkgkgugyyurhjfffhfjklhhhgkjgGGaaKKint}
\begin{cases}
-\Delta_{\vec x}\vec A \approx\frac{4\pi}{\kappa_0 c}\vec {\widetilde j},\\
-\Delta_{\vec x}\Psi_1=4\pi\gamma_0\rho+\frac{1}{c}\,div_{\vec x}
\left\{\left(d_{\vec x}\vec {\tilde u}+\left\{d_{\vec x}\vec {\tilde
u}\right\}^T\right)\cdot\vec A-\left(div_{\vec x}\vec {\tilde
u}\right)\vec A\right\},
\end{cases}
\end{equation}
where we set the reduced current:
\begin{equation}\label{reducedcurrentfhfhjfhjGGaaint}
\begin{cases}
\vec {\widetilde j}:=\vec j-\frac{1}{4\pi}\frac{\partial}{\partial
t} \left(\nabla_{\vec x}\psi_0\right)+\frac{1}{4\pi}curl_{\vec
x}\left(\vec {\tilde u}\times \nabla_{\vec x}\psi_0\right),\\
-\Delta_{\vec x}\psi_0= 4\pi\rho.
\end{cases}
\end{equation}
Note that by the Continuum Equation of the Conservation of Charges:
\begin{equation}\label{toksohraneniezarjadaPPNaaint}
\frac{\partial\rho}{\partial t}+div_{\vec x}\vec j\equiv 0,
\end{equation}
the reduced current clearly satisfies:
\begin{equation}\label{divreducedcurrentPPNaaint}
div_{\vec x}\vec {\widetilde j}\equiv 0.
\end{equation}
Moreover,
using Proposition
\ref{yghgjtgyrtrtint}
we can easily deduce that $\vec{\widetilde j}$ is a proper vector
field (see subsection \ref{qfCM} for detatls). Finally, the
approximate vectorial electromagnetic potential $\vec A$ from
\er{MaxVacFull1bjkgjhjhgjgjgkjfhjfdghcgjhhjgkgkgugyyurhjfffhfjklhhhgkjgGGaaKKint}
clearly satisfies:
\begin{equation}\label{MaxVacFull1bjkgjhjhgjgjgkjfhjfdghcgjhhjgkgkgugyyurkkkGGGGGaaint}
div_{\vec x}\vec A=0.
\end{equation}
Next, since by \er{vhfffngghhjghhgPPNghghghutghffugghjhjkjjklggint}
we have
\begin{equation}\label{vhfffngghhjghhgPPNghghghutghffugghjhjkjjklgghkhhhint}
\Psi_1:=\Psi-\frac{1}{c}\vec A\cdot\vec {\tilde u},
\end{equation}
%
%
%
%
%
%
we rewrite
\er{MaxVacFull1bjkgjhjhgjgjgkjfhjfdghcgjhhjgkgkgugyyurhjfffhfjklhhhgkjgGGaaKKint}
as:
\begin{equation}\label{MaxVacFull1bjkgjhjhgjgjgkjfhjfdghcgjhhjgkgkgugyyurhjfffhfjklhhhgkjgGGaaint}
\begin{cases}
-\Delta_{\vec x}\vec A \approx\frac{4\pi}{\kappa_0 c}\vec {\widetilde j},\\
-\Delta_{\vec x}\Psi= 4\pi\gamma_0\rho-\frac{1}{c}\,div_{\vec
x}\left(\vec {\tilde u}\times curl_{\vec x} \vec A\right).
\end{cases}
\end{equation}
where
\begin{equation}\label{reducedcurrentfhfhjfhjGGaakklklint}
\begin{cases}
\vec {\widetilde j}:=\vec j-\frac{1}{4\pi}\frac{\partial}{\partial
t} \left(\nabla_{\vec x}\psi_0\right)+\frac{1}{4\pi}curl_{\vec
x}\left(\vec {\tilde u}\times \nabla_{\vec x}\psi_0\right),\\
-\Delta_{\vec x}\psi_0= 4\pi\rho,
\end{cases}
\end{equation}
(see subsection \ref{qfCM} for details). So in order to find the
scalar and the vectorial electromagnetic potentials we just need to
solve Laplace equations. Knowing the approximate electromagnetic
potentials by \er{vhfffngghPPN333yuyuGGint} we can find the
approximations of of the electromagnetic fields:
\begin{equation}\label{MaxVacFull1bjkgjhjhgjgjgkjfhjfdghghligioiuittrhiguffGGaaint}
\begin{cases}
\vec B= curl_{\vec x} \vec A\\
\vec E=-\nabla_{\vec x}\Psi-\frac{1}{c}\frac{\partial\vec
A}{\partial t}\\
 \vec D=-\frac{1}{\gamma_0}\nabla_{\vec
x}\Psi-\frac{1}{\gamma_0 c}\frac{\partial\vec A}{\partial t}+\frac{1}{c\gamma_0}\vec {\tilde u}\times curl_{\vec x}\vec A\\
\vec H=\kappa_0 \,curl_{\vec x} \vec A+\frac{1}{c}\,\vec {\tilde
u}\times\left(-\frac{1}{\gamma_0}\nabla_{\vec
x}\Psi-\frac{1}{\gamma_0 c}\frac{\partial\vec A}{\partial
t}+\frac{1}{\gamma_0 c}\vec {\tilde u}\times curl_{\vec x}\vec
A\right),
\end{cases}
\end{equation}
where $\Psi$ and $\vec A$ are given by
\er{MaxVacFull1bjkgjhjhgjgjgkjfhjfdghcgjhhjgkgkgugyyurhjfffhfjklhhhgkjgGGaaint}.
Note also that, since $\vec {\widetilde j}$ is a proper vector
field, by Proposition \ref{yghgjtgyrtrtint} we deduce that equations
\er{MaxVacFull1bjkgjhjhgjgjgkjfhjfdghcgjhhjgkgkgugyyurhjfffhfjklhhhgkjgGGaaKKint}
and thus also equations
\er{MaxVacFull1bjkgjhjhgjgjgkjfhjfdghcgjhhjgkgkgugyyurhjfffhfjklhhhgkjgGGaaint}
are invariant under the change of non-inertial cartesian coordinate
system, provided that $\vec A$ is a proper vector field and
$\Psi_1=\Psi-\frac{1}{c}\vec A\cdot\vec {\tilde u}$ is a proper
scalar field. So the approximate solutions in the case of
quasistationary fields in a weak gravitational field satisfy the
same transformation as the exact solutions of Maxwell Equations.
Therefore, if in coordinate system $(*)$ we can use the approximate
equations, given by
\er{MaxVacFull1bjkgjhjhgjgjgkjfhjfdghcgjhhjgkgkgugyyurhjfffhfjklhhhgkjgGGaaint}
and
\er{MaxVacFull1bjkgjhjhgjgjgkjfhjfdghghligioiuittrhiguffGGaaint},
then we can use the similar approximation
also in coordinate system $(**)$, even in the case when in system
$(**)$ \er{slowaetherGGaaint} or \er{slochangGGaaint} are not
satisfied.
\begin{remark}
The solutions of
\er{MaxVacFull1bjkgjhjhgjgjgkjfhjfdghcgjhhjgkgkgugyyurhjfffhfjklhhhgkjgGGaaint}
and \er{MaxVacFull1bjkgjhjhgjgjgkjfhjfdghghligioiuittrhiguffGGaaint}
satisfy the following equations:
\begin{equation}\label{MaxVacFull1bjkgjhjhgjaaajhfghhgGGaaint}
\begin{cases}
curl_{\vec x} \left(\kappa_0\vec B+\frac{1}{c}\,\vec {\tilde
u}\times \left(- \nabla_{\vec x}\psi_0\right)\right)\equiv
\frac{4\pi}{c}\vec j+\frac{1}{c}\frac{\partial (-
\nabla_{\vec x}\psi_0)}{\partial t},\\
div_{\vec x} \vec D=4\pi\rho,\\
curl_{\vec x} \vec E+\frac{1}{c}\frac{\partial \vec B}{\partial t}=0,\\
div_{\vec x} \vec B=0\\
\vec E=\gamma_0\vec D-\frac{1}{c}\,\vec {\tilde u}\times \vec B\\
\vec H=\kappa_0\vec B+\frac{1}{c}\,\vec {\tilde u}\times \vec D,\\
\vec {\tilde u}=\left(\gamma_0\kappa_0\vec
v+(1-\gamma_0\kappa_0)\vec u\right),
\end{cases}
\end{equation}
where $\psi_0$ was defined by \er{columbPPNaaint}. Equations
\er{MaxVacFull1bjkgjhjhgjaaajhfghhgGGaaint} differ from the original
Maxwell equations \er{MaxVacFullPPNffGGint} only by neglecting the
divergence-free part of the vector field $\vec D$ on the first
equation.
\end{remark}

Next, assume that, in addition to the validity of approximation
\er{slowaetherGGaaint} and \er{slochangGGaaint}, the approximation
\er{MaxVacFullPPNmmmffffffhhtygghGGGGyuhggghghint} also holds. Then
we further approximate
\er{MaxVacFull1bjkgjhjhgjgjgkjfhjfdghcgjhhjgkgkgugyyurhjfffhfjklhhhgkjgGGaaKKint}
as:
\begin{equation}\label{MaxVacFull1bjkgjhjhgjgjgkjfhjfdghcgjhhjgkgkgugyyurhjfffhfjklhhhgkjgGGaaKKjkjjint}
\begin{cases}
-\Delta_{\vec x}\Psi_1=4\pi\gamma_0\rho,\\
-\Delta_{\vec x}\vec A \approx \frac{4\pi}{\kappa_0 c}\,\vec
j-\frac{1}{\kappa_0\gamma_0 c}\left(\frac{\partial}{\partial
t}\left(\nabla_{\vec
x}\Psi_1\right)-curl_{\vec x}\left(\vec {\tilde u}\times\nabla_{\vec x}\Psi_1\right)\right)\\
\Psi=\Psi_1+\frac{1}{c}\vec A\cdot\vec {\tilde u}.
\end{cases}
\end{equation}
Moreover, as before, we deduce that equations
\er{MaxVacFull1bjkgjhjhgjgjgkjfhjfdghcgjhhjgkgkgugyyurhjfffhfjklhhhgkjgGGaaKKjkjjint}
are also invariant under the change of non-inertial cartesian
coordinate system. Therefore, as before, if in coordinate system
$(*)$ we can use the approximation equations, given by
\er{MaxVacFull1bjkgjhjhgjgjgkjfhjfdghcgjhhjgkgkgugyyurhjfffhfjklhhhgkjgGGaaKKjkjjint}
then we can use the similar equations
also in coordinate system $(**)$, even in the case when in system
$(**)$ \er{slowaetherGGaaint}, \er{slochangGGaaint} or
\er{MaxVacFullPPNmmmffffffhhtygghGGGGyuhggghghint} are not
satisfied.

Finally, assume that we are under the alternative calibration
\er{MaxVacFullPPNjjjjffhhGGint}. Then by \er{slowaetherGGaaint} and
\er{slochangGGaaint} we rewrite
\er{MaxVacFullPPNmmmffffffiuiuhjuGGint} and
\er{MaxVacFullPPNnnnffffffyuughjhjhjhhjjkjhkkjhujgGGint} as:
\begin{equation}\label{MaxVacFullPPNmmmffffffiuiuhjuGGaaint}
-\Delta_{\vec x}\Psi_1\approx
4\pi\gamma_0\rho+\frac{1}{c}\,div_{\vec x} \left\{\left(d_{\vec
x}\vec {\tilde u}+\left\{d_{\vec x}\vec {\tilde
u}\right\}^T\right)\cdot\vec A-\left(div_{\vec x}\vec {\tilde
u}\right)\vec A\right\},
\end{equation}
and
\begin{equation}\label{MaxVacFullPPNnnnffffffyuughjhjhjhhjjkjhkkjhujgGGaaint}
-\Delta_{\vec x}\vec A\approx \frac{4\pi}{\kappa_0 c}\left(\vec
j-\rho\vec {\tilde u}\right)+\frac{1}{\kappa_0\gamma_0
c}\left(\left(d_{\vec x}\vec {\tilde u}+\left\{d_{\vec x}\vec
{\tilde u}\right\}^T\right)\cdot \nabla_{\vec
x}\Psi_1-\left(div_{\vec x}\vec {\tilde u}\right)\nabla_{\vec
x}\Psi_1\right).
\end{equation}
Thus if we assume that in addition to the approximation
\er{slowaetherGGaaint} and \er{slochangGGaaint} the approximation
\er{MaxVacFullPPNmmmffffffhhtygghGGGGyuhggghghint} also holds, we
further approximate \er{MaxVacFullPPNmmmffffffiuiuhjuGGaaint} and
\er{MaxVacFullPPNnnnffffffyuughjhjhjhhjjkjhkkjhujgGGaaint} as:
\begin{equation}\label{MaxVacFullPPNmmmffffffiuiuhjuGGaaghgghghint}
\begin{cases}
-\Delta_{\vec x}\Psi_1\approx 4\pi\gamma_0\rho,
\\
-\Delta_{\vec x}\vec A\approx \frac{4\pi}{\kappa_0 c}\left(\vec
j-\rho\vec {\tilde u}\right)\\
\Psi=\Psi_1+\frac{1}{c}\vec A\cdot\vec {\tilde u}.
\end{cases}
\end{equation}
Moreover, as before, we deduce that equations
\er{MaxVacFullPPNmmmffffffiuiuhjuGGaaghgghghint} are also invariant
under the change of non-inertial cartesian coordinate system.
Therefore, as before, if in coordinate system $(*)$ we can use the
approximation equations, given by
\er{MaxVacFullPPNmmmffffffiuiuhjuGGaaghgghghint}
then we can use the similar equations
also in coordinate system $(**)$, even in the case when in system
$(**)$ \er{slowaetherGGaaint}, \er{slochangGGaaint} or
\er{MaxVacFullPPNmmmffffffhhtygghGGGGyuhggghghint} are not
satisfied.

\subsubsection{Hertz's notation for quasistationary electromagnetic fields}\label{gcCMuhhjhgghint} Again consider in some domain the full system of
Maxwell equations in the vacuum or in a medium, without dispersion,
having the form \er{MaxVacFullPPNffGGint}:
\begin{equation}\label{MaxVacFullPPNffGGvbccgcint}
\begin{cases}
curl_{\vec x} \vec H=\frac{4\pi}{c}\vec j+
\frac{1}{c}\frac{\partial \vec D}{\partial t},\\
div_{\vec x} \vec D=4\pi\rho,\\
curl_{\vec x} \vec E+\frac{1}{c}\frac{\partial \vec B}{\partial t}=0,\\
div_{\vec x} \vec B=0\\
\vec E=\gamma_0\vec D-\frac{1}{c}\,\vec {\tilde u}\times \vec B\\
\vec H=\kappa_0\vec B+\frac{1}{c}\,\vec {\tilde u}\times \vec D,\\
\vec {\tilde u}=\left(\gamma_0\kappa_0\vec
v+(1-\gamma_0\kappa_0)\vec u\right).
\end{cases}
\end{equation}
Next define the electromagnetic fields in Hertz's notation:
\begin{equation}\label{MaxVacFullPPNffGGvbccgcjhhhjjihihjhint}
\begin{cases}
\vec{\tilde{ E}}:=\vec E+\frac{1}{c}\,\vec {u}\times \vec B,\\
\vec{\tilde{ D}}:=\frac{1}{\gamma_0}\vec{\tilde{ E}}=\frac{1}{\gamma_0}\left(\vec E+\frac{1}{c}\,\vec {u}\times \vec B\right),\\
\vec{\tilde{ B}}:=\frac{1}{\kappa_0}\left(\vec H-\frac{1}{c}\,\vec {u}\times \vec D\right),\\
\vec{\tilde{ H}}:=\kappa_0\vec{\tilde{ B}}=\left(\vec
H-\frac{1}{c}\,\vec {u}\times \vec D\right).
\end{cases}
\end{equation}
Note here that $\vec u$ is the velocity field in the medium, thus
since this field is absent in vacuum we use the convention that for
the vacuum we have $\vec u:=\vec v$ (and in addition $\vec {\tilde
u}=\vec v$). In particular, by
\er{MaxVacFullPPNffGGvbccgcjhhhjjihihjhint} and
\er{MaxVacFullPPNffGGvbccgcint} we have
\begin{equation}\label{MaxVacFullPPNffGGvbccgcjhhhjjihihjhhhhiuu8hjiyint}
\begin{cases}
\vec{\tilde{ E}}=\gamma_0\vec D+\frac{1}{c}\,\left(\vec {u}-\vec {\tilde u}\right)\times \vec B=\gamma_0\left(\vec D+\frac{\kappa_0}{c}\,\left(\vec {u}-\vec v\right)\times \vec B\right),\\
\vec{\tilde{ D}}=\vec D+\frac{\kappa_0}{c}\,\left(\vec {u}-\vec v\right)\times \vec B,\\
\vec{\tilde{ B}}=\frac{1}{\kappa_0}\left(\kappa_0\vec B-\frac{1}{c}\,\left(\vec {u}-\vec {\tilde u}\right)\times \vec D\right)=\vec B-\frac{\gamma_0}{c}\,\left(\vec {u}-\vec v\right)\times \vec D,\\
\vec{\tilde{ H}}=\kappa_0\vec
B-\frac{\gamma_0\kappa_0}{c}\,\left(\vec {u}-\vec v\right)\times
\vec D.
\end{cases}
\end{equation}
Therefore, since $\vec D,\vec B, (\vec u-\vec v)$ are all proper
vector fields, we deduce that $\vec{\tilde{ E}}, \vec{\tilde{ D}},
\vec{\tilde{ B}}, \vec{\tilde{ H}}$ are all \underline{proper}
vector fields.
Next, analogously with \er{slowaetherGGaaint}, assume  that in the
given inertial or non-inertial cartesian coordinate system $(*)$ the
fields $\vec u$, $\vec v$, $\vec {\tilde u}$ are weak, meaning that
at any instant on every point:
\begin{equation}\label{slowaetherGGaagghghint}
\frac{1}{\kappa_0\gamma_0}\left(\frac{|\vec {
u}|^2}{c^2}+\frac{|\vec v|^2}{c^2}+\frac{|\vec {\tilde
u}|^2}{c^2}\right)\,\ll\, 1,
\end{equation}
and consider quasistationary electromagnetic fields, meaning the
following: assume that the changes in time of the physical
characteristics of the electromagnetic fields become essential after
certain interval of time $T_e$ and the changes in space of the
physical characteristics of the fields become essential in the
spatial landscape $L_e$. Then, analogously with
\er{slochangGGaaint}, we assume that
\begin{equation}\label{slochangGGaajhhjhjint}
(\kappa_0\gamma_0)\frac{c^2T^2_e}{L^2_e}\,\gg\, 1.
\end{equation}
Moreover, assume that we have the following approximation analogous
to \er{MaxVacFullPPNmmmffffffhhtygghGGGGyuhggghghint}: if the
changes in space of the physical characteristics of the
electromagnetic fields become essential in the spatial landscape
$L_e$ and the changes in space of the fields $\vec u$, $\vec v$,
$\vec {\tilde u}$ become essential in the spatial landscape $L_{u}$,
then we assume
\begin{equation}\label{MaxVacFullPPNmmmffffffhhtygghGGGGyuhggghghffggfgfint}
L_e\ll L_u,
\end{equation}
i.e. the fields $\vec u$, $\vec v$, $\vec {\tilde u}$ change in
space much slower than the electromagnetic fields.
%
%
%
%
%
%
%
%
Then, denoting
\begin{equation}\label{MaxVacFullPPNffGGvbccgcjhhhjjihihjhkjkhkhhkint}
\begin{cases}
\tilde\rho^*:=\rho-\frac{1}{c^2}\left(\vec {u}-\vec
v\right)\cdot\left(\vec j-\rho\vec u\right),\\
\vec {\tilde j}^*=\left(\vec j-\rho\vec u\right)+\tilde\rho^*\vec u,
\end{cases}
\end{equation}
by \er{slowaetherGGaagghghint}, \er{slochangGGaajhhjhjint} and
\er{MaxVacFullPPNmmmffffffhhtygghGGGGyuhggghghffggfgfint}, using the
equality of the conservation of the charge
$\frac{\partial\rho}{\partial t}+div_{\vec x}\vec j=0$ we deduce:
\begin{equation}\label{MaxVacFullPPNffGGvbccgcjhhhjjihihjhkjkhkhhkhkhhjggint}
\begin{cases}
\vec j\approx\vec {\tilde j^*},
\\
\frac{\partial\tilde\rho^*}{\partial t}+div_{\vec x}\vec {\tilde
j^*}\approx 0.
\end{cases}
\end{equation}
Moreover, by \er{MaxVacFullPPNffGGvbccgcjhhhjjihihjhkjkhkhhkint} we
easily obtain that $\tilde\rho^*$ is a proper scalar and $\left(\vec
{\tilde j^*}-\tilde\rho^*\vec u\right)$ is a proper vector.
%
%
%
%
%
%
%
%
Finally, by \er{slowaetherGGaagghghint}, \er{slochangGGaajhhjhjint},
\er{MaxVacFullPPNmmmffffffhhtygghGGGGyuhggghghffggfgfint} and
\er{MaxVacFullPPNffGGvbccgcjhhhjjihihjhkjkhkhhkint},
\er{MaxVacFullPPNffGGvbccgcjhhhjjihihjhkjkhkhhkhkhhjggint},
\er{MaxVacFullPPNffGGvbccgcjhhhjjihihjhint} we deduce from the
original equations \er{MaxVacFullPPNffGGvbccgcint} the following
approximate Maxwell equations in Hertz's form and write them
as
the following two sets of equations:
\begin{equation}\label{MaxVacFullPPNffGGvbccgcu8uuuuukjjkhhhihjhjkhhkmgjhjhhkhjhhjjhjhjjhint}
\begin{cases}
curl_{\vec x} \left(\kappa_0\vec{\tilde{
B}}\right)\approx\frac{4\pi}{c}\vec {\tilde j^*}+
\frac{1}{c}\frac{\partial }{\partial
t}\left(\frac{1}{\gamma_0}\vec{\tilde{
E}}\right)-\frac{1}{c}\,curl_{\vec x}\left(\vec
{u}\times\frac{1}{\gamma_0} \vec{\tilde{
E}}\right),\\
div_{\vec x} \left(\frac{1}{\gamma_0}\vec{\tilde{ E}}\right)\approx
4\pi\tilde\rho^*,\\
curl_{\vec x} \vec{\tilde{
E}}\approx-\left(\frac{1}{c}\frac{\partial \vec{\tilde{
B}}}{\partial t}-\frac{1}{c}\,curl_{\vec x}\left(\vec {u}\times
\vec{\tilde{
B}}\right)\right),\\
div_{\vec x} \vec{\tilde{ B}}\approx 0,\\
\frac{\partial\tilde\rho^*}{\partial t}+div_{\vec x}\vec {\tilde
j^*}\approx 0,
\end{cases}
\end{equation}
and
\begin{equation}\label{MaxVacFullPPNffGGvbccgcu8uuuuukjjkhhhihjhjkhhkmgjhjhhkhjhhjjhjhjjh11int}
\begin{cases}
\rho=\tilde\rho^*+\frac{1}{c^2}\left(\vec {u}-\vec
v\right)\cdot\left(\vec {\tilde j^*}-\tilde\rho^*\vec u\right),\\
\vec j=\left(\vec {\tilde j^*}-\tilde\rho^*\vec u\right)+\rho\vec u,
\\
\vec B\approx\vec{\tilde{ B}}+\frac{1}{c}\,\left(\vec {u}-\vec
v\right)\times \vec{\tilde{ E}}\\
\vec E=\vec{\tilde{ E}}-\frac{1}{c}\,\vec {u}\times \vec{ B},
\end{cases}
\end{equation}
(see subsection \ref{gcCMuhhjhggh} for the details) that are valid,
however, only for quasistationary electromagnetic fields. Next,
since $\vec{\tilde{ E}}, \vec{\tilde{ D}}, \vec{\tilde{ B}},
\vec{\tilde{ H}}$, $\vec B,\left(\vec E+\frac{1}{c}\,\vec {u}\times
\vec{ B}\right)$ and $\left(\vec {\tilde j^*}-\tilde\rho^*\vec
u\right),\left(\vec {j}-\rho\vec u\right),\left(\vec {u}-\vec
v\right) $ are all proper vector fields and $\tilde\rho^*,\rho$ are
proper scalars, as before we deduce that all equations in
\er{MaxVacFullPPNffGGvbccgcu8uuuuukjjkhhhihjhjkhhkmgjhjhhkhjhhjjhjhjjhint}
and
\er{MaxVacFullPPNffGGvbccgcu8uuuuukjjkhhhihjhjkhhkmgjhjhhkhjhhjjhjhjjh11int}
are invariant under the change of inertial or non-inertial cartesian
coordinate system. Therefore, if in coordinate system $(*)$ we can
use the approximate equations, given by
\er{MaxVacFullPPNffGGvbccgcu8uuuuukjjkhhhihjhjkhhkmgjhjhhkhjhhjjhjhjjhint}
and
\er{MaxVacFullPPNffGGvbccgcu8uuuuukjjkhhhihjhjkhhkmgjhjhhkhjhhjjhjhjjh11int},
then we can use the similar approximation
also in coordinate system $(**)$, even in the case when in system
$(**)$ \er{slowaetherGGaagghghint}, \er{slochangGGaajhhjhjint} or
\er{MaxVacFullPPNmmmffffffhhtygghGGGGyuhggghghffggfgfint} are not
satisfied.
Next, by \er{vjhfhjtjhjhuyyiyGGint} we adjoint
\er{MaxVacFullPPNffGGvbccgcu8uuuuukjjkhhhihjhjkhhkmgjhjhhkhjhhjjhjhjjhint}
and
\er{MaxVacFullPPNffGGvbccgcu8uuuuukjjkhhhihjhjkhhkmgjhjhhkhjhhjjhjhjjh11int}
with the Ohm's Law that in our notations has the following simple
form:
\begin{equation}\label{vjhfhjtjhjhuyyiyGGhhgjgint}
\vec {\tilde j^*}-\tilde\rho^*\,\vec u=\varepsilon\vec{\tilde E}\,,
\end{equation}
and the Joules heat term, appearing in the First Law of
Thermodynamics
\er{MaxVacFull1ninshtrgravortghhghgjkgghklhjgkghghjjkjhjkkggjkhjkhjjhhfhjhkjkhbbgjhzzzyyykkknnnNWBWHWPPNjjhjhjhjhint},
in our notations has the following simple form
\begin{equation}\label{vjhiiiihhjjhjhint}\left(\vec j-\rho\vec u\right)\cdot\left(\vec E+\frac{1}{c}\vec
u\times\vec B\right)= \left(\vec {\tilde j^*}-\tilde\rho^*\,\vec
u\right)\cdot\vec{\tilde E}\,.
\end{equation}

Finally, note that the system
\er{MaxVacFullPPNffGGvbccgcu8uuuuukjjkhhhihjhjkhhkmgjhjhhkhjhhjjhjhjjhint}
is completely analogous to the system
\er{MaxVacFullPPNffGGvbccgcint}, where the field
$\frac{1}{\gamma_0}\vec{\tilde{ E}}$ substitutes the field $\vec D$,
the field $\vec{\tilde{ B}}$ substitutes the field $\vec B$,
$\tilde\rho^*,\vec{\tilde j^*}$ substitute $\rho,\vec{j}$, and with
the real velocity of the medium $\vec u$ instead of the field $\vec
{\tilde u}=\left(\gamma_0\vec v+(1-\gamma_0)\vec u\right)$. On the
other hand, in the case of the real medium (not vacuum) the
equations in
\er{MaxVacFullPPNffGGvbccgcu8uuuuukjjkhhhihjhjkhhkmgjhjhhkhjhhjjhjhjjhint}
are completely independent of the vectorial gravitation potential
$\vec v$. Since the systems
\er{MaxVacFullPPNffGGvbccgcu8uuuuukjjkhhhihjhjkhhkmgjhjhhkhjhhjjhjhjjhint}
and \er{MaxVacFullPPNffGGvbccgcint} are similar they both obeys
wave-type solutions and thus one can wonder here: is it possible
that
\er{MaxVacFullPPNffGGvbccgcu8uuuuukjjkhhhihjhjkhhkmgjhjhhkhjhhjjhjhjjhint}
is approximately equivalent to \er{MaxVacFullPPNffGGvbccgcint} also
for highly oscillating in time electromagnetic fields and/or for
wave-type solutions in the absence of charges and currents, for
example in the case of the visible light? Unfortunately, the answer
for the last question should be negative, indeed as it can be seen,
the term $(1-\gamma_0)\vec u$ (consistently with the equality $\vec
{\tilde u}=\left(\gamma_0\vec v+(1-\gamma_0)\vec u\right)$) appear
in the the Fizeau experiment for the light (see subsection
\ref{seGOfzint}) rather than the term $\vec u$. The main reason for
it is that \er{slochangGGaajhhjhjint} is not valid for highly
oscillating in time electromagnetic fields. Moreover, in the case of
the complete absence of charges and currents in the whole space
\er{slochangGGaajhhjhjint} is not valid also for any nontrivial
wave-type solution, even with moderate frequency. So
\er{MaxVacFullPPNffGGvbccgcu8uuuuukjjkhhhihjhjkhhkmgjhjhhkhjhhjjhjhjjhint}
and
\er{MaxVacFullPPNffGGvbccgcu8uuuuukjjkhhhihjhjkhhkmgjhjhhkhjhhjjhjhjjh11int}
are valid, only for quasistationary electromagnetic fields,
generated by charges and currents moderately changing in time, for
example in the case of estimation of DC or AC chains.

Next, as before, again by the third and the fourth equations in
\er{MaxVacFullPPNffGGvbccgcu8uuuuukjjkhhhihjhjkhhkmgjhjhhkhjhhjjhjhjjhint}
we can write
\begin{equation}\label{MaxVacFull1bjkgjhjhgjgjgkjfhjfdghghligioiuittrPPNggjoiuouiint}
\begin{cases}
\vec {\tilde B}\approx curl_{\vec x} \vec {\tilde A^*},\\
\vec {\tilde E}\approx -\nabla_{\vec
x}\tilde\Psi^*-\frac{1}{c}\frac{\partial\vec {\tilde A^*}}{\partial
t}+\frac{1}{c}\,\vec {u}\times \left(curl_{\vec x} \vec {\tilde
A^*}\right),
\end{cases}
\end{equation}
where $\tilde\Psi^*$ and $\vec {\tilde A^*}$ are the scalar and the
vectorial electromagnetic potentials in Hertz's notation. Next,
analogously to \er{vhfffngghhjghhgPPNghghghutghffugghjhjkjjklggint}
we define the proper scalar electromagnetic potential in Hertz's
notation:
\begin{equation}\label{vhfffngghhjghhgPPNghghghutghffugghjhjkjjklgghihhihint}
\tilde\Psi^*_1:=\tilde\Psi^*-\frac{1}{c}\vec {\tilde A^*}\cdot\vec
{u}.
\end{equation}
As before, the electromagnetic potentials in Hertz's notation are
not uniquely defined and thus we need to choose a calibration. For
definiteness we can take $\vec {\tilde A^*}$ to satisfy
\begin{equation}\label{MaxVacFull1bjkgjhjhgjgjgkjfhjfdghghligioiuittrPPN22jnhhjhjhint}
div_{\vec x}\vec {\tilde A^*}= 0.
\end{equation}
Then, as before, $\vec {\tilde A^*}$ is a proper vector field and
$\tilde\Psi^*_1$ is a proper scalar field. It is clear that any
other choice of electromagnetic potentials with a different
calibration, as before, can be obtained by
\begin{equation}\label{MaxVacFull1bjkgjhjhgjgjgkjfhjfdghghligioiuittrPPNhjkjhkjgghhjjhjjhjgjint}
\begin{cases}
\tilde\Psi^*\,\rightarrow\, \tilde\Psi^*+\frac{1}{c}\frac{\partial w}{\partial t}\\
\vec {\tilde A^*}\,\rightarrow\, \vec {\tilde A^*}-\nabla_{\vec x}w\\
\tilde\Psi^*_1\,\rightarrow\,
\tilde\Psi^*_1+\frac{1}{c}\left(\frac{\partial w}{\partial t}+\vec
u\cdot\nabla_{\vec x}w\right).
\end{cases}
\end{equation}
where $w:=w(\vec x,t)$ is an arbitrary proper scalar field.
Moreover, as before, under such a change of calibration, $\vec
{\tilde A^*}$ is still a proper vector field and $\tilde\Psi^*_1$ is
still a proper scalar field, provided the scalar field $w$ is
\underline{proper}. In particular the following alternative to
\er{MaxVacFull1bjkgjhjhgjgjgkjfhjfdghghligioiuittrPPN22jnhhjhjhint}
calibration of the potentials can be chosen:
\begin{equation}\label{MaxVacFullPPNjjjjffhhGGljkjjkkjhkint}
\frac{1}{\kappa_0\gamma_0
c}\left(\frac{\partial\tilde\Psi^*_1}{\partial t}+\vec
{u}\cdot\nabla_{\vec x}\tilde\Psi^*_1\right)+div_{\vec x}\vec
{\tilde A^*}=0.
\end{equation}

Next, inserting
\er{vhfffngghhjghhgPPNghghghutghffugghjhjkjjklgghihhihint} into
\er{MaxVacFull1bjkgjhjhgjgjgkjfhjfdghghligioiuittrPPNggjoiuouiint}
gives
\begin{equation}\label{MaxVacFull1bjkgjhjhgjgjgkjfhjfdghghligioiuittrPPNggjoiuouighujjhmngkhint}
\begin{cases}
\vec {\tilde B}\approx curl_{\vec x} \vec {\tilde A^*},\\
\vec {\tilde E}\approx -\nabla_{\vec
x}\tilde\Psi^*_1-\frac{1}{c}\left(\frac{\partial\vec {\tilde
A^*}}{\partial t}-\vec {u}\times \left(curl_{\vec x} \vec {\tilde
A^*}\right)+\nabla_{\vec x}\left(\vec {\tilde A^*}\cdot\vec
{u}\right)\right).
\end{cases}
\end{equation}
%
%
%
%
%
%
%
%
Furthermore, inserting
\er{MaxVacFull1bjkgjhjhgjgjgkjfhjfdghghligioiuittrPPNggjoiuouighujjhmngkhint}
into
\er{MaxVacFullPPNffGGvbccgcu8uuuuukjjkhhhihjhjkhhkmgjhjhhkhjhhjjhjhjjhint}
and using \er{slowaetherGGaagghghint}, \er{slochangGGaajhhjhjint}
and \er{MaxVacFullPPNmmmffffffhhtygghGGGGyuhggghghffggfgfint} we
deduce
\begin{equation}\label{MaxVacFullPPNffGGvbccgcu8uuuuukjjkhhhihjhjkhhkmgjhjhhkhjhhjjhjhjjhjkkhhhhjint}
\begin{cases}
curl_{\vec x} \left(\kappa_0\vec{\tilde{
B}}\right)\approx\frac{4\pi}{c}\left(\vec {\tilde
j^*}-\tilde\rho^*\vec u\right)- \frac{1}{c}\left(\frac{\partial
}{\partial t}\left(\frac{1}{\gamma_0}\nabla_{\vec
x}\tilde\Psi^*_1\right)-\,curl_{\vec x}\left(\vec
{u}\times\frac{1}{\gamma_0} \nabla_{\vec
x}\tilde\Psi^*_1\right)+\left(div_{\vec x}\left\{\frac{1}{\gamma_0}
\nabla_{\vec
x}\tilde\Psi^*_1\right\}\right)\vec u\right),\\
div_{\vec x} \left(\frac{1}{\gamma_0}\vec{\tilde{ E}}\right)\approx
4\pi\tilde\rho^*,\\
\frac{\partial\tilde\rho^*}{\partial t}+div_{\vec x}\vec {\tilde
j^*}\approx 0,\\
\vec {\tilde E}\approx-\nabla_{\vec
x}\tilde\Psi^*_1-\frac{1}{c}\left(\frac{\partial\vec {\tilde
A^*}}{\partial t}-\,curl_{\vec x}\left(\vec {u}\times \vec {\tilde
A^*}\right)+\left(div_{\vec x}\vec {\tilde A^*}\right)\vec
{u}\right),
\\ \vec {\tilde B}\approx curl_{\vec x} \vec {\tilde
A^*},
\end{cases}
\end{equation}
(see subsection \ref{gcCMuhhjhggh} for the details). Then again all
equations in the new system
\er{MaxVacFullPPNffGGvbccgcu8uuuuukjjkhhhihjhjkhhkmgjhjhhkhjhhjjhjhjjhjkkhhhhjint}
are invariant under the change of inertial or non-inertial cartesian
coordinate system.

Furthermore, if we assume that in the domain of the study the
coefficients $\gamma_0\neq 0$ and $\kappa_0\neq 0$ vary sufficiently
slow on the place and time and thus their spatial and temporal
derivatives are negligible then
by \er{slowaetherGGaagghghint}, \er{slochangGGaajhhjhjint} and
\er{MaxVacFullPPNmmmffffffhhtygghGGGGyuhggghghffggfgfint},
\er{MaxVacFullPPNffGGvbccgcu8uuuuukjjkhhhihjhjkhhkmgjhjhhkhjhhjjhjhjjhjkkhhhhjint}
can be approximately rewritten as:
\begin{equation}\label{MaxVacFullPPNffGGvbccgcu8uuuuukjjkhhhihjhjkhhkmgjhjhhkhjhhjjhjhjjhjjjkghhhfint}
\begin{cases}
-\Delta_{\vec x}\vec {\tilde
A^*}\approx\frac{4\pi}{\kappa_0c}\left(\vec {\tilde
j^*}-\tilde\rho^*\vec u\right)- \nabla_{\vec x}\left\{div_{\vec
x}\vec {\tilde
A^*}+\frac{1}{c}\frac{1}{\gamma_0\kappa_0}\left(\frac{\partial\tilde\Psi^*_1}{\partial
t}+\vec
{u}\cdot\nabla_{\vec x}\tilde\Psi^*_1\right)\right\},\\
-\Delta_{\vec x}\tilde\Psi_1\approx
4\pi\gamma_0\tilde\rho^*+\frac{1}{c}\left(\frac{\partial}{\partial
t}\left(div_{\vec x}\vec {\tilde A^*}\right)+div_{\vec
x}\left\{\left(div_{\vec x}\vec {\tilde A^*}\right)\vec
{u}\right\}\right),\\
\frac{\partial\tilde\rho^*}{\partial t}+div_{\vec x}\vec {\tilde
j^*}\approx 0,\\
\vec {\tilde E}\approx -\nabla_{\vec
x}\tilde\Psi^*_1-\frac{1}{c}\left(\frac{\partial\vec {\tilde
A^*}}{\partial t}-\vec {u}\times \left(curl_{\vec x} \vec {\tilde
A^*}\right)+\nabla_{\vec x}\left(\vec {\tilde A^*}\cdot\vec
{u}\right)\right),\\
\vec {\tilde B}\approx curl_{\vec x} \vec {\tilde A^*},
\end{cases}
\end{equation}
(again see subsection \ref{gcCMuhhjhggh} for the details).
Furthermore, defining the conduction currents:
\begin{equation}\label{uiyuyyyyyyyyyyyyyyyyyyyyyyyyyyyyyyyyyyyyyyyyyyyyyyykljint}
\begin{cases}
\vec {j}_{cond}:=\vec {j}-\rho\vec u,\\
\vec {\tilde j^*}_{cond}:=\vec {\tilde j^*}-\tilde\rho^*\vec u,
\end{cases}
\end{equation}
we obviously have:
\begin{equation}\label{uiyuyyyyyyyyyyyyyyyyyyyyyyyyyyyyyyyyyyyyyyyyyyyyyyykljggngnint}
\begin{cases}
\vec {j}_{cond}= \vec {\tilde j^*}_{cond},\\
\frac{\partial\rho}{\partial t}=div_{\vec x}\left\{\rho\vec
u\right\}=-div_{\vec x}\left\{\vec {j}_{cond}\right\},
\\
\frac{\partial\tilde\rho^*}{\partial t}+div_{\vec
x}\left\{\tilde\rho^*\vec u\right\}\approx-div_{\vec x}\left\{\vec
{\tilde j^*}_{cond}\right\},
\end{cases}
\end{equation}
and moreover, obviously $\vec {j}_{cond}=\vec {\tilde j^*}_{cond}$
is a \underline{proper} vector field. Thus, in the case of
calibration
\er{MaxVacFull1bjkgjhjhgjgjgkjfhjfdghghligioiuittrPPN22jnhhjhjhint}
by \er{vfyutuyfffhhgfhgfhint} and \er{vhfffngghhjghhgjlkhjhkPPPint}
in Proposition \ref{yghgjtgyrtrtint}, using
\er{MaxVacFullPPNmmmffffffhhtygghGGGGyuhggghghffggfgfint},
\er{MaxVacFullPPNffGGvbccgcu8uuuuukjjkhhhihjhjkhhkmgjhjhhkhjhhjjhjhjjhjjjkghhhfint}
can be rewritten as:
\begin{equation}\label{MaxVacFullPPNffGGvbccgcu8uuuuukjjkhhhihjhjkhhkmgjhjhhkhjhhjjhjhjjhjjjkghhhfhjijhjhkhhkint}
\begin{cases}
-\Delta_{\vec x}\tilde\Psi_1\approx
4\pi\gamma_0\tilde\rho^*,\\
-\Delta_{\vec x}\vec {\tilde A^*}\approx\frac{4\pi}{\kappa_0c}\vec
{\tilde j^*}_{cond}
-\frac{1}{c}\frac{1}{\gamma_0\kappa_0}\left(\frac{\partial
}{\partial t}\left(\nabla_{\vec x}\tilde\Psi^*_1\right)-\,curl_{\vec
x}\left(\vec {u}\times\nabla_{\vec
x}\tilde\Psi^*_1\right)+\left(\Delta_{\vec
x}\tilde\Psi^*_1\right)\vec u\right)
,\\
\frac{\partial\tilde\rho^*}{\partial t}+div_{\vec
x}\left\{\tilde\rho^*\vec u\right\}\approx-div_{\vec x}\left\{\vec
{\tilde j^*}_{cond}\right\},\\
\vec {\tilde E}\approx -\nabla_{\vec
x}\tilde\Psi^*_1-\frac{1}{c}\left(\frac{\partial\vec {\tilde
A^*}}{\partial t}-\vec {u}\times \left(curl_{\vec x} \vec {\tilde
A^*}\right)+\nabla_{\vec x}\left(\vec {\tilde A^*}\cdot\vec
{u}\right)\right),
\\
\vec {\tilde B}\approx curl_{\vec x} \vec {\tilde A^*}.
\end{cases}
\end{equation}
On the other hand, under the alternative calibration
\er{MaxVacFullPPNjjjjffhhGGljkjjkkjhkint}, again using
\er{slowaetherGGaagghghint}, \er{slochangGGaajhhjhjint} and
\er{MaxVacFullPPNmmmffffffhhtygghGGGGyuhggghghffggfgfint} in the
second equation of
\er{MaxVacFullPPNffGGvbccgcu8uuuuukjjkhhhihjhjkhhkmgjhjhhkhjhhjjhjhjjhjjjkghhhfint},
we approximately rewrite
\er{MaxVacFullPPNffGGvbccgcu8uuuuukjjkhhhihjhjkhhkmgjhjhhkhjhhjjhjhjjhjjjkghhhfint}
as:
\begin{equation}\label{MaxVacFullPPNffGGvbccgcu8uuuuukjjkhhhihjhjkhhkmgjhjhhkhjhhjjhjhjjhjjjkghhhfjkjkkhjggjint}
\begin{cases}
-\Delta_{\vec x}\vec {\tilde A^*}\approx\frac{4\pi}{\kappa_0c}\vec
{\tilde j^*}_{cond},\\
-\Delta_{\vec x}\tilde\Psi_1\approx
4\pi\gamma_0\tilde\rho^*,\\
\frac{\partial\tilde\rho^*}{\partial t}+div_{\vec
x}\left\{\tilde\rho^*\vec u\right\}\approx-div_{\vec x}\left\{\vec
{\tilde j^*}_{cond}\right\},\\
\vec {\tilde E}\approx -\nabla_{\vec
x}\tilde\Psi^*_1-\frac{1}{c}\left(\frac{\partial\vec {\tilde
A^*}}{\partial t}-\vec {u}\times \left(curl_{\vec x} \vec {\tilde
A^*}\right)+\nabla_{\vec x}\left(\vec {\tilde A^*}\cdot\vec
{u}\right)\right),
\\
\vec {\tilde B}\approx curl_{\vec x} \vec {\tilde A^*}.
\end{cases}
\end{equation}
In both given cases of calibration in
\er{MaxVacFullPPNffGGvbccgcu8uuuuukjjkhhhihjhjkhhkmgjhjhhkhjhhjjhjhjjhjjjkghhhfhjijhjhkhhkint}
or
\er{MaxVacFullPPNffGGvbccgcu8uuuuukjjkhhhihjhjkhhkmgjhjhhkhjhhjjhjhjjhjjjkghhhfjkjkkhjggjint}
we just need to solve two Laplace equations in order to find
$\tilde\Psi_1$ and $\vec {\tilde A^*}$ and then we find $\vec
{\tilde E}$ and $\vec {\tilde B}$ by differentiation. Moreover, in
 either the case of
\er{MaxVacFullPPNffGGvbccgcu8uuuuukjjkhhhihjhjkhhkmgjhjhhkhjhhjjhjhjjhjjjkghhhfhjijhjhkhhkint}
or the case of
\er{MaxVacFullPPNffGGvbccgcu8uuuuukjjkhhhihjhjkhhkmgjhjhhkhjhhjjhjhjjhjjjkghhhfjkjkkhjggjint}
we have
\begin{equation}\label{MaxVagjkgggggggggggggjhjgint}
div_{\vec x}\left\{\vec {\tilde E}+\nabla_{\vec
x}\tilde\Psi^*_1\right\}\approx 0.
\end{equation}
and the Ohm's Law in the form of \er{vjhfhjtjhjhuyyiyGGhhgjgint}:
\begin{equation}\label{vjhfhjtjhjhuyyiyGGhhgjgkhhhhhjint}
\vec {\tilde j^*}_{cond}=\varepsilon\vec{\tilde E}.
\end{equation}
Furthermore, the Joules heat term, has the following simple form
\begin{equation}\label{vjhiiiihhjjhjhkkkint}
\vec {\tilde j^*}_{cond}\cdot\vec{\tilde E}\,.
\end{equation}
Moreover, in both cases of
\er{MaxVacFullPPNffGGvbccgcu8uuuuukjjkhhhihjhjkhhkmgjhjhhkhjhhjjhjhjjhjjjkghhhfhjijhjhkhhkint}
or
\er{MaxVacFullPPNffGGvbccgcu8uuuuukjjkhhhihjhjkhhkmgjhjhhkhjhhjjhjhjjhjjjkghhhfjkjkkhjggjint},
by the last two equations there we obviously have:
\begin{equation}\label{gjkgjkgjkgjkgjkgjkgjkgjkgjkgjkgjkgjkgjkgjkgjkgjkgjkgjkgjkgjkgjkhjjhint}
\begin{cases}
curl_{\vec x} \vec{\tilde{
E}}\approx-\left(\frac{1}{c}\frac{\partial \vec{\tilde{
B}}}{\partial t}-\frac{1}{c}\,curl_{\vec x}\left(\vec {u}\times
\vec{\tilde{
B}}\right)\right),\\
div_{\vec x} \vec{\tilde{ B}}\approx 0.
\end{cases}
\end{equation}

Next, in order to find the real electromagnetic fields in the usual
notation we have
\er{MaxVacFullPPNffGGvbccgcu8uuuuukjjkhhhihjhjkhhkmgjhjhhkhjhhjjhjhjjh11int},
i.e. the following:
\begin{equation}\label{MaxVacFullPPNffGGvbccgcu8uuuuukjjkhhhihjhjkhhkmgjhjhhkhjhhjjhjhjjh11ljkjkljkljint}
\begin{cases}
\vec B\approx\vec{\tilde{ B}}+\frac{1}{c}\,\left(\vec {u}-\vec
v\right)\times \vec{\tilde{ E}},\\
\vec E=\vec{\tilde{ E}}-\frac{1}{c}\,\vec {u}\times \vec{ B},\\
\rho=\tilde\rho^*+\frac{1}{c^2}\left(\vec {u}-\vec v\right)\cdot\vec
{\tilde j^*}_{cond},\\
\vec {j}_{cond}= \vec {\tilde j^*}_{cond},\\
\vec j=\vec {\tilde j^*}_{cond}+\rho\vec u.
\end{cases}
\end{equation}
Finally, as before, we deduce that all equations in either
\er{MaxVacFullPPNffGGvbccgcu8uuuuukjjkhhhihjhjkhhkmgjhjhhkhjhhjjhjhjjhjjjkghhhfhjijhjhkhhkint}
or
\er{MaxVacFullPPNffGGvbccgcu8uuuuukjjkhhhihjhjkhhkmgjhjhhkhjhhjjhjhjjhjjjkghhhfjkjkkhjggjint},
and, in addition, \er{MaxVagjkgggggggggggggjhjgint},
\er{vjhfhjtjhjhuyyiyGGhhgjgkhhhhhjint},
\er{gjkgjkgjkgjkgjkgjkgjkgjkgjkgjkgjkgjkgjkgjkgjkgjkgjkgjkgjkgjkgjkhjjhint}
and
\er{MaxVacFullPPNffGGvbccgcu8uuuuukjjkhhhihjhjkhhkmgjhjhhkhjhhjjhjhjjh11ljkjkljkljint}
are invariant under the change of inertial or non-inertial cartesian
coordinate system. Therefore, if in coordinate system $(*)$ we can
use the approximate equations, given by either
\er{MaxVacFullPPNffGGvbccgcu8uuuuukjjkhhhihjhjkhhkmgjhjhhkhjhhjjhjhjjhjjjkghhhfhjijhjhkhhkint}
or
\er{MaxVacFullPPNffGGvbccgcu8uuuuukjjkhhhihjhjkhhkmgjhjhhkhjhhjjhjhjjhjjjkghhhfjkjkkhjggjint},
and \er{MaxVagjkgggggggggggggjhjgint},
\er{vjhfhjtjhjhuyyiyGGhhgjgkhhhhhjint},
\er{gjkgjkgjkgjkgjkgjkgjkgjkgjkgjkgjkgjkgjkgjkgjkgjkgjkgjkgjkgjkgjkhjjhint}
and
\er{MaxVacFullPPNffGGvbccgcu8uuuuukjjkhhhihjhjkhhkmgjhjhhkhjhhjjhjhjjh11ljkjkljkljint},
then we can use the similar approximation
also in coordinate system $(**)$, even in the case when in system
$(**)$ \er{slowaetherGGaagghghint}, \er{slochangGGaajhhjhjint} or
\er{MaxVacFullPPNmmmffffffhhtygghGGGGyuhggghghffggfgfint} are not
satisfied.

Next, note that either
\er{MaxVacFullPPNffGGvbccgcu8uuuuukjjkhhhihjhjkhhkmgjhjhhkhjhhjjhjhjjhjjjkghhhfhjijhjhkhhkint}
or
\er{MaxVacFullPPNffGGvbccgcu8uuuuukjjkhhhihjhjkhhkmgjhjhhkhjhhjjhjhjjhjjjkghhhfjkjkkhjggjint},
and \er{MaxVagjkgggggggggggggjhjgint},
\er{vjhfhjtjhjhuyyiyGGhhgjgkhhhhhjint} and
\er{gjkgjkgjkgjkgjkgjkgjkgjkgjkgjkgjkgjkgjkgjkgjkgjkgjkgjkgjkgjkgjkhjjhint}
are completely independent of the vectorial gravitational potential
$\vec v$. Moreover, note that the first two equations in
\er{MaxVacFullPPNffGGvbccgcu8uuuuukjjkhhhihjhjkhhkmgjhjhhkhjhhjjhjhjjhjjjkghhhfjkjkkhjggjint}
are completely independent also on the velocity field of the medium
$\vec u$. Furthermore, if $\Omega:=\Omega(t)\subset\mathbb{R}^3$ is
a three-dimensional domain, moving with velocity field $\vec u$
together with the given medium, then, by the third equation of
either
\er{MaxVacFullPPNffGGvbccgcu8uuuuukjjkhhhihjhjkhhkmgjhjhhkhjhhjjhjhjjhjjjkghhhfhjijhjhkhhkint}
or
\er{MaxVacFullPPNffGGvbccgcu8uuuuukjjkhhhihjhjkhhkmgjhjhhkhjhhjjhjhjjhjjjkghhhfjkjkkhjggjint},
using part {{\bf (iii)}} of Proposition
\ref{jgjjjjjjjjjjjjjjjjjjjjjjjjjjjjjjjjjkhhkhint} and the Divergence
Integral Theorem of the Calculus we deduce:
%
%
%
%
%
%
%
%
\begin{equation}\label{hjgghffhhffgfgfghkhjhjggjkhkhbnv44jhggggggggggghkjhint}
\frac{d}{dt}\left(\iiint\rho\, d\Omega(t)\right)=-\iint\vec
{j}_{cond}\,\cdot\vec n\, d\left(\partial\Omega(t)\right)=-\iint\vec
{\tilde j^*}_{cond}\,\cdot\vec n\,
d\left(\partial\Omega(t)\right)\approx
\frac{d}{dt}\left(\iiint\tilde\rho^*\, d\Omega(t)\right).
\end{equation}
Thus, denoting
\begin{equation}\label{khjjhjhgjgjggjint}
\begin{cases}
Q_{\Omega(t)}(t):=\iiint\rho\, d\Omega(t),\quad \tilde Q^*_{\Omega(t)}(t):=\iiint\tilde\rho^*\, d\Omega(t)\quad\text{and}\quad\\
I_{\partial\Omega(t)}(t):=\iint\vec {j}_{cond}\,\cdot\vec n
=\iint\vec {\tilde j^*}_{cond}\,\cdot\vec n\,
d\left(\partial\Omega(t)\right):=\tilde I^*_{\partial\Omega(t)}(t),
\end{cases}
\end{equation}
which are the total charge inside the domain $\Omega(t)$, the total
charge$^*$ inside the domain $\Omega(t)$ and the corresponding total
current outward the boundary of $\Omega(t)$, by
\er{hjgghffhhffgfgfghkhjhjggjkhkhbnv44jhggggggggggghkjhint} we have
\begin{equation}\label{hjgghffhhffgfgfghkhjhjggjkhkhbnv44jhggggggggggghkjhkjhjhjhjhjhint}
\frac{d}{dt}\left(Q_{\Omega(t)}(t)\right)=-I_{\partial\Omega(t)}(t)=-\tilde
I^*_{\partial\Omega(t)}(t)\approx \frac{d}{dt}\left(\tilde
Q^*_{\Omega(t)}(t)\right).
\end{equation}
On the other hand, if $\gamma:=\gamma(t)\subset\mathbb{R}^3$ is a
one-dimensional curve oriented by the unit tangent vector $\vec
t:=\vec t(\vec x,t)$, having the starting and the ending points
$\vec r_{begin}(t), \vec r_{end}(t)$ and moving with velocity field
$\vec u$ together with the given medium, then, by the fourth
equation in either
\er{MaxVacFullPPNffGGvbccgcu8uuuuukjjkhhhihjhjkhhkmgjhjhhkhjhhjjhjhjjhjjjkghhhfhjijhjhkhhkint}
or
\er{MaxVacFullPPNffGGvbccgcu8uuuuukjjkhhhihjhjkhhkmgjhjhhkhjhhjjhjhjjhjjjkghhhfjkjkkhjggjint}
using part {{\bf (ii)}} of Proposition
\ref{jgjjjjjjjjjjjjjjjjjjjjjjjjjjjjjjjjjkhhkhint} we obtain:
%
%
%
%
%
%
%
%
\begin{equation}\label{MaxVacFullPPNffGGvbccgcu8uuuuukjjkhhhihjhjkhhkmgjhjhhkhjhhjjhjhjjh11ljkjkljkljjjkjkljkjk11jhgggjh11gjgghhjhhhint}
\int\vec {\tilde E}\cdot\vec t\,
d\gamma(t)\approx\tilde\Psi^*_1\left(\vec
r_{begin}(t),t\right)-\tilde\Psi^*_1\left(\vec
r_{end}(t),t\right)-\frac{1}{c}\frac{d}{dt}\left(\int\vec {\tilde
A^*}\cdot\vec t\, d\gamma(t)\right).
\end{equation}
In particular, in the case where $\gamma(t):=\partial\mathcal{S}(t)$
is a boundary of a two-dimensional surface
$\mathcal{S}:=\mathcal{S}(t)\subset\mathbb{R}^3$, oriented by the
unit normal $\vec n:=\vec n(\vec x,t)$, since in the later case we
have $\vec r_{begin}(t)=\vec r_{end}(t)$, using the last equation in
either
\er{MaxVacFullPPNffGGvbccgcu8uuuuukjjkhhhihjhjkhhkmgjhjhhkhjhhjjhjhjjhjjjkghhhfhjijhjhkhhkint}
or
\er{MaxVacFullPPNffGGvbccgcu8uuuuukjjkhhhihjhjkhhkmgjhjhhkhjhhjjhjhjjhjjjkghhhfjkjkkhjggjint}
and the Stokes Theorem, we rewrite
\er{MaxVacFullPPNffGGvbccgcu8uuuuukjjkhhhihjhjkhhkmgjhjhhkhjhhjjhjhjjh11ljkjkljkljjjkjkljkjk11jhgggjh11gjgghhjhhhint}
as
\begin{equation}\label{MaxVacFullPPNffGGvbccgcu8uuuuukjjkhhhihjhjkhhkmgjhjhhkhjhhjjhjhjjh11ljkjkljkljjjkjkljkjk11jhgggjh11gjgghhjhhhhjjhjghjjhhjhjhint}
\int\vec {\tilde E}\cdot\vec t\,
d\gamma(t)\approx-\frac{1}{c}\frac{d}{dt}\left(\iint\vec {\tilde
B}\cdot\vec n\,
d\mathcal{S}(t)\right)=-\frac{1}{c}\frac{d}{dt}\left(\iint
curl_{\vec x}\vec {\tilde A^*}\cdot\vec n\, d\mathcal{S}(t)\right).
\end{equation}
Finally, the obtained relations in
\er{MaxVagjkgggggggggggggjhjgint},
\er{vjhfhjtjhjhuyyiyGGhhgjgkhhhhhjint}, \er{vjhiiiihhjjhjhkkkint},
\er{hjgghffhhffgfgfghkhjhjggjkhkhbnv44jhggggggggggghkjhkjhjhjhjhjhint}
with \er{khjjhjhgjgjggjint},
\er{MaxVacFullPPNffGGvbccgcu8uuuuukjjkhhhihjhjkhhkmgjhjhhkhjhhjjhjhjjh11ljkjkljkljjjkjkljkjk11jhgggjh11gjgghhjhhhint},
\er{MaxVacFullPPNffGGvbccgcu8uuuuukjjkhhhihjhjkhhkmgjhjhhkhjhhjjhjhjjh11ljkjkljkljjjkjkljkjk11jhgggjh11gjgghhjhhhhjjhjghjjhhjhjhint}
and the first two equations in
\er{MaxVacFullPPNffGGvbccgcu8uuuuukjjkhhhihjhjkhhkmgjhjhhkhjhhjjhjhjjhjjjkghhhfjkjkkhjggjint}
are completely independent of the velocity field $\vec u$ ( and on
the vectorial gravitational potential $\vec v$). However, the
mentioned relations in \er{MaxVagjkgggggggggggggjhjgint},
\er{vjhfhjtjhjhuyyiyGGhhgjgkhhhhhjint},
\er{hjgghffhhffgfgfghkhjhjggjkhkhbnv44jhggggggggggghkjhkjhjhjhjhjhint}
with \er{khjjhjhgjgjggjint},
\er{MaxVacFullPPNffGGvbccgcu8uuuuukjjkhhhihjhjkhhkmgjhjhhkhjhhjjhjhjjh11ljkjkljkljjjkjkljkjk11jhgggjh11gjgghhjhhhint},
\er{MaxVacFullPPNffGGvbccgcu8uuuuukjjkhhhihjhjkhhkmgjhjhhkhjhhjjhjhjjh11ljkjkljkljjjkjkljkjk11jhgggjh11gjgghhjhhhhjjhjghjjhhjhjhint}
and the first two equations in
\er{MaxVacFullPPNffGGvbccgcu8uuuuukjjkhhhihjhjkhhkmgjhjhhkhjhhjjhjhjjhjjjkghhhfjkjkkhjggjint}
are together sufficient for estimation of DC or AC linear chains
and, in particular, they sufficient to obtain the Kirchhoff's
current and voltage laws, Faraday's law of induction and the law of
capacity. Thus, the estimation of DC or quasistationary AC currents
in the case of moving electrical chains and/or in the case of
non-trivial gravitation is completely analogous to that estimation
for the \underline{resting} chains without influence of any
gravitation!

\subsection{Geometric optics inside a moving medium and/or in the presence of gravitational
field}
\subsubsection{Preliminary calculations}
Assume that in some inertial or non-inertial cartesian coordinate
system a scalar \underline{real valued} field $U:=U(\vec x,t)$,
characterizing some wave, satisfies the following wave equation
\begin{equation}\label{MaxVacFullPPNmmmffffffiuiuhjuughbghhjjint}
\frac{\partial}{\partial
t}\left\{\frac{1}{c^2_0}\left(\frac{\partial U}{\partial t}+\vec
{\tilde u}\cdot\nabla_{\vec x}U\right)\right\}+\Div_{\vec x}
\left\{\frac{1}{c^2_0}\left(\frac{\partial U}{\partial t}+\vec
{\tilde u}\cdot\nabla_{\vec x} U\right)\vec {\tilde
u}\right\}-\Delta_{\vec x}U=0,
\end{equation}
where $\vec {\tilde u}:=\vec {\tilde u}(\vec x,t)$ is some
moderately changing (in space and in time) speed-like vector field
and $c_0:=c_0(\vec x,t)>0$ is a moderately changing (in space and in
time) scalar quantity, that we call wave propagation speed. Note
that \er{MaxVacFullPPNmmmffffffiuiuhjuughbghhjjint} coincides with
\er{MaxVacFullPPNmmmffffffiuiuhjuGGggFGelGHGHGHGGint} and thus, in
particular, $U$ can represent one of any scalar components of the
electromagnetic field in the medium without dispersion, i.e. when
\er{EBDHTrans444knbuihuigGGint} is valid. In this case $\vec{\tilde
u}$ is the linear combination of the velocity of the medium $\vec u$
and the vectorial gravitational potential $\vec v$. Note also that
\er{MaxVacFull1ninshtrgravortghhghgjkgghklhjgkghghjjkjhjkkggjkhjkhjjhhfhjhkjhjjhkhbbgjlmkmmhzzzyyykkknnnNWBWHWPPNjjlllkkkkkkljkkjjkjkhjjhhgghgggghjjhgjgghint}
or
\er{MaxVacFull1ninshtrgravortghhghgjkgghklhjgkghghjjkjhjkkggjkhjkhjjhhfhjhkjhjjhkhbbgjlmkmmhzzzyyykkknnnNWBWHWPPNjjlllkkkkkkljkkjjkjkhjjhhgghgggghjjhgjgghjkjkjkint}
also coincide with \er{MaxVacFullPPNmmmffffffiuiuhjuughbghhjjint}
and thus, in particular, $U$ could represent the oscillating part of
the pressure $p_1$ in the sound wave. However, in the later case
$\vec{\tilde u}$ is equal to the averaged (macroscopical) velocity
of the fluid/gas $\vec u_0$. Moreover,
\er{MaxVacFull1ninshtrgravortghhghgjkgghklhjgkghghjjkjhjkkggjkhjkhjjhhfhjhkjkhbbgjhzzzyyykkknnnNWBWHWPPNjjjljkjkkhhhhhhhhhhhhfnfffhuyuyuyutjhjhjhggjjfgfhfhjhhgjgjhkjhjhkhkhjjint}
and \er{MaxVacFullPPNmmmffffffiuiuhjuGGggFGjjkkjjffjjint} also
coincide with \er{MaxVacFullPPNmmmffffffiuiuhjuughbghhjjint} and
thus, in particular, $U$ could represent either longitudinal or
transverse wave in an elastic body. In the later case $\vec{\tilde
u}$ is also equal to the averaged (macroscopical) velocity of the
elastic medium $\vec u_0$.

Next if we assume that the fields $\vec {\tilde u}$ and $c_0$ are
independent of the time variable, since  $U$ is a real valued field,
then we can write the field $U$ as a Fourier's Transform on the time
variable:
\begin{equation}\label{MaxVacFullPPNmmmffffffhhtygghGGGGyuhggghghghffggggint}
U(\vec x,t)=
Re\left\{2\int_{0}^{+\infty} \hat U(\vec x,\omega)e^{i\omega
t}d\omega\right\}\quad\text{where}\quad\hat U(\vec
x,\omega):=\frac{1}{2\pi}\int_{-\infty}^{+\infty} U(\vec
x,t)e^{-i\omega t}dt\,.
\end{equation}
Moreover, by \er{MaxVacFullPPNmmmffffffiuiuhjuughbghhjjint} we
obtain that the Fourier's Transform $\hat U(\vec x,\omega)$
satisfies:
\begin{equation}\label{MaxVacFullPPNmmmffffffiuiuhjuughbghhjjuggint}
\frac{i\omega}{c^2_0}\left(i\omega \hat U+\vec {\tilde
u}\cdot\nabla_{\vec x}\hat U\right)+\Div_{\vec x}
\left\{\frac{1}{c^2_0}\left(i\omega \hat U+\vec {\tilde
u}\cdot\nabla_{\vec x} \hat U\right)\vec {\tilde
u}\right\}-\Delta_{\vec x}\hat U=0.
\end{equation}
Thus by \er{MaxVacFullPPNmmmffffffiuiuhjuughbghhjjuggint}, for every
given $\omega$ the monochromatic wave type function
\begin{equation}\label{MaxVacFullPPNmmmffffffhhtygghGGGGyuhggghghghffgggguiuint}
U_\omega(\vec x,t):=\hat U(\vec x,\omega)e^{i\omega t}
\end{equation}
is a complex solution of
\begin{equation}\label{MaxVacFullPPNmmmffffffiuiuhjuughbghhjjghghhint}
\frac{\partial}{\partial
t}\left\{\frac{1}{c^2_0}\left(\frac{\partial U_\omega}{\partial
t}+\vec {\tilde u}\cdot\nabla_{\vec
x}U_\omega\right)\right\}+\Div_{\vec x}
\left\{\frac{1}{c^2_0}\left(\frac{\partial U_\omega}{\partial
t}+\vec {\tilde u}\cdot\nabla_{\vec x} U_\omega\right)\vec {\tilde
u}\right\}-\Delta_{\vec x}U_\omega=0.
\end{equation}
Note that equation
\er{MaxVacFullPPNmmmffffffiuiuhjuughbghhjjghghhint} coincides with
\er{MaxVacFullPPNmmmffffffiuiuhjuughbghhjjint}. Moreover, by
\er{MaxVacFullPPNmmmffffffhhtygghGGGGyuhggghghghffggggint} a general
solution of \er{MaxVacFullPPNmmmffffffiuiuhjuughbghhjjint} can be
represented as a superposition of monochromatic waves of type
$U_\omega=f(\vec x,\omega)e^{i\omega t}$ that satisfy
\er{MaxVacFullPPNmmmffffffiuiuhjuughbghhjjghghhint} for every
$\omega$. Finally if we consider a complex valued function
$\mathcal{U}(\vec x,t)$, defined by
\begin{equation}\label{MaxVacFullPPNmmmffffffhhtygghGGGGyuhggghghghffggggjkjjkjkint}
\mathcal{U}(\vec x,t):=2\int_{0}^{+\infty} \hat U(\vec
x,\omega)e^{i\omega t}d\omega,
\end{equation}
then $Re\,\mathcal{U}=U$ and $\mathcal{U}$ is a complex solution of
\er{MaxVacFullPPNmmmffffffiuiuhjuughbghhjjint} (i.e. not only the
real part of $\mathcal{U}$ but also the imaginary part solve
\er{MaxVacFullPPNmmmffffffiuiuhjuughbghhjjint}).

Next assume that a scalar \underline{complex} field $U:=U(\vec x,t)$
satisfies \er{MaxVacFullPPNmmmffffffiuiuhjuughbghhjjint}. In
particular, $U$ can be a monochromatic solution of
\er{MaxVacFullPPNmmmffffffiuiuhjuughbghhjjghghhint}.
%
%
%
%
%
%
Furthermore, we represent the complex field $U$ as:
\begin{equation}\label{MaxVacFullPPNmmmffffffiuiuhjuughbghhuiiujjjjjjjjint}
U(\vec x,t)=A(\vec x,t)e^{iT(\vec x,t)},
\end{equation}
where $A:=A(\vec x,t)$ and $T:=T(\vec x,t)$ are real scalar fields.
Then define
\begin{equation}\label{MaxVacFullPPNmmmffffffiuiuhjuughbghhuiiujjjjjjjjhhhjjjint}
\omega:=\left<\,\left|\frac{\partial T}{\partial t}\right|\,\right>,
\end{equation}
where the sign $\left<\cdot\right>$ means the spatial and temporal
averaging. Next define $k_0$ and a scalar field $S:=S(\vec x,t)$ by
\begin{equation}\label{MaxVacFullPPNmmmffffffiuiuhjuughbghhuiiujjjjjjjjhhhjjjkkint}
k_0:=\frac{\omega}{c}\quad\quad\text{and}\quad\quad S(\vec
x,t)=\frac{1}{k_0}T(\vec x,t),
\end{equation}
where $c$ is a constant in the Maxwell equations for the vacuum. So
we clearly have
\begin{equation}\label{MaxVacFullPPNmmmffffffiuiuhjuughbghhuiiujjint}
U(\vec x,t)=A(\vec x,t)e^{ik_0S(\vec x,t)}\,.
\end{equation}
We in position to insert it into the wave equation
\er{MaxVacFullPPNmmmffffffiuiuhjuughbghhjjint} of the form
\begin{equation}\label{MaxVacFullPPNmmmffffffiuiuhjuughbghhjjjoojouiint}
\left(\frac{\partial}{\partial
t}\left\{\frac{1}{c^2_0}\left(\frac{\partial U}{\partial t}+\vec
{\tilde u}\cdot\nabla_{\vec x}U\right)\right\}+\Div_{\vec x}
\left\{\frac{1}{c^2_0}\left(\frac{\partial U}{\partial t}+\vec
{\tilde u}\cdot\nabla_{\vec x} U\right)\vec {\tilde
u}\right\}\right)-\Delta_{\vec x}U=0.
\end{equation}
Then, inserting \er{MaxVacFullPPNmmmffffffiuiuhjuughbghhuiiujjint}
into \er{MaxVacFullPPNmmmffffffiuiuhjuughbghhjjjoojouiint} we
deduce:
\begin{multline}\label{MaxVacFullPPNmmmffffffiuiuhjuughbghhiuijghghghhjhjhjjuiuiuiint}
k^2_0\left(\left|\nabla_{\vec
x}S\right|^2-\frac{1}{c^2_0}\left(\frac{\partial S}{\partial t}+\vec
{\tilde u}\cdot\nabla_\vec x
S\right)^2\right)A\\+\left(\frac{\partial}{\partial
t}\left\{\frac{1}{c^2_0}\left(\frac{\partial A}{\partial t}+\vec
{\tilde u}\cdot\nabla_{\vec x}A\right)\right\}+\Div_{\vec x}
\left\{\frac{1}{c^2_0}\left(\frac{\partial A}{\partial t}+\vec
{\tilde u}\cdot\nabla_{\vec x} A\right)\vec {\tilde
u}\right\}\right)-\Delta_{\vec
x}A\\
+i\kappa_0 A\left(\frac{\partial}{\partial
t}\left\{\frac{1}{c^2_0}\left(\frac{\partial S}{\partial t}+\vec
{\tilde u}\cdot\nabla_{\vec x}S\right)\right\}+\Div_{\vec x}
\left\{\frac{1}{c^2_0}\left(\frac{\partial S}{\partial t}+\vec
{\tilde u}\cdot\nabla_{\vec x} S\right)\vec {\tilde
u}\right\}-\Delta_{\vec
x}S\right)\\+2i\kappa_0\left(\frac{1}{c^2_0}\left(\frac{\partial
S}{\partial t}+\vec {\tilde u}\cdot\nabla_{\vec
x}S\right)\left(\frac{\partial A}{\partial t}+\vec {\tilde
u}\cdot\nabla_{\vec x}A\right)-\nabla_{\vec x}A\cdot\nabla_{\vec
x}S\right)=0
\end{multline}
(see subsection \ref{GO} for details).

\subsubsection{Derivation of the Eikonal equation}
Comparing both real and imaginary part of
\er{MaxVacFullPPNmmmffffffiuiuhjuughbghhiuijghghghhjhjhjjuiuiuiint}
to zero we obtain two equations:
\begin{multline}\label{MaxVacFullPPNmmmffffffiuiuhjuughbghhiuijghghghhjhjhjjint}
k^2_0\left(\left|\nabla_{\vec
x}S\right|^2-\frac{1}{c^2_0}\left(\frac{\partial S}{\partial t}+\vec
{\tilde u}\cdot\nabla_\vec x
S\right)^2\right)A\\+\left(\frac{\partial}{\partial
t}\left\{\frac{1}{c^2_0}\left(\frac{\partial A}{\partial t}+\vec
{\tilde u}\cdot\nabla_{\vec x}A\right)\right\}+\Div_{\vec x}
\left\{\frac{1}{c^2_0}\left(\frac{\partial A}{\partial t}+\vec
{\tilde u}\cdot\nabla_{\vec x} A\right)\vec {\tilde
u}\right\}\right)-\Delta_{\vec x}A=0,
%
%
%
\end{multline}
and
\begin{multline}\label{MaxVacFullPPNmmmffffffiuiuhjuughbghhiuijghghghhhfhjjint}
A\left(\frac{\partial}{\partial
t}\left\{\frac{1}{c^2_0}\left(\frac{\partial S}{\partial t}+\vec
{\tilde u}\cdot\nabla_{\vec x}S\right)\right\}+\Div_{\vec x}
\left\{\frac{1}{c^2_0}\left(\frac{\partial S}{\partial t}+\vec
{\tilde u}\cdot\nabla_{\vec x} S\right)\vec {\tilde
u}\right\}-\Delta_{\vec
x}S\right)\\+\frac{2}{c^2_0}\left(\frac{\partial S}{\partial t}+\vec
{\tilde u}\cdot\nabla_{\vec x}S\right)\left(\frac{\partial
A}{\partial t}+\vec {\tilde u}\cdot\nabla_{\vec
x}A\right)-2\nabla_{\vec x}A\cdot\nabla_{\vec x}S=0
\end{multline}

Next assume the Geometric Optics approximation that is good for the
electromagnetic wave of high frequency for example for the visible
light. The delicate Geometric Optics approximation means the
following: assume that the changes in time of $c_0$, $\vec {\tilde
u}$, $A$ and $S$ become essential after certain interval of time
$T_e$ and the changes in space of $c_0$, $\vec {\tilde u}$, $A$ and
$S$ become essential in the spatial landscape $L_e$. Then we assume
that
\begin{equation}\label{slochangGGaaffgfgjhjgghint}
k^2_0c^2_0T^2_e\,\gg\, 1\quad\text{and}\quad k^2_0L^2_e\,\gg\, 1.
\end{equation}
On the other hand the rough Geometric Optics approximation (stronger
than \er{slochangGGaaffgfgjhjgghint}) means
\begin{equation}\label{slochangGGaaffgfgjhjggh11int}
k_0c_0T_e\,\gg\, 1\quad\text{and}\quad k_0L_e\,\gg\, 1.
\end{equation}
Thus, by \er{slochangGGaaffgfgjhjgghint} we approximate
\er{MaxVacFullPPNmmmffffffiuiuhjuughbghhiuijghghghhjhjhjjint} as the
Eikonal-type equation:
\begin{equation}\label{MaxVacFullPPNmmmffffffiuiuhjuughbghhiuijghghghhjhjhhghyuyjjjhhjhjffint}
\frac{1}{c^2_0}\left(\frac{\partial S}{\partial t}+\vec {\tilde
u}\cdot\nabla_\vec x S\right)^2=\left|\nabla_\vec x S\right|^2,
\end{equation}
and, without any use in \er{slochangGGaaffgfgjhjgghint}, we rewrite
\er{MaxVacFullPPNmmmffffffiuiuhjuughbghhiuijghghghhhfhjjint} as:
\begin{multline}\label{MaxVacFullPPNmmmffffffiuiuhjuughbghhiuijghghghhhfhhghghguygtjjjkkjjkkhhint}
A\left(\frac{\partial}{\partial
t}\left\{\frac{1}{c^2_0}\left(\frac{\partial S}{\partial t}+\vec
{\tilde u}\cdot\nabla_{\vec x}S\right)\right\}+\Div_{\vec x}
\left\{\frac{1}{c^2_0}\left(\frac{\partial S}{\partial t}+\vec
{\tilde u}\cdot\nabla_{\vec x} S\right)\vec {\tilde
u}\right\}-\Delta_{\vec
x}S\right)\\+\frac{2}{c^2_0}\left(\frac{\partial S}{\partial t}+\vec
{\tilde u}\cdot\nabla_{\vec x}S\right)\left(\frac{\partial
A}{\partial t}+\vec {\tilde u}\cdot\nabla_{\vec
x}A\right)-2\nabla_{\vec x}A\cdot\nabla_{\vec x}S=0\,.
\end{multline}
Equality
\er{MaxVacFullPPNmmmffffffiuiuhjuughbghhiuijghghghhjhjhhghyuyjjjhhjhjffint}
is called the Eikonal equation and equality
\er{MaxVacFullPPNmmmffffffiuiuhjuughbghhiuijghghghhhfhhghghguygtjjjkkjjkkhhint}
is called the equation of the beam propagation. Then, as before, we
deduce that equation
\er{MaxVacFullPPNmmmffffffiuiuhjuughbghhiuijghghghhjhjhhghyuyjjjhhjhjffint}
is invariant under the change of non-inertial cartesian coordinate
system, provided that under such change we have $S'=S$. Moreover,
\er{MaxVacFullPPNmmmffffffiuiuhjuughbghhiuijghghghhhfhhghghguygtjjjkkjjkkhhint}
is also invariant under the change of non-inertial cartesian
coordinate system, in the case that under such change we have
$A'=A$, provided that $S'=S$. So if the approximation
\er{slochangGGaaffgfgjhjgghint} is valid in some cartesian
coordinate system $(*)$, then we can use
\er{MaxVacFullPPNmmmffffffiuiuhjuughbghhiuijghghghhjhjhhghyuyjjjhhjhjffint}
and
\er{MaxVacFullPPNmmmffffffiuiuhjuughbghhiuijghghghhhfhhghghguygtjjjkkjjkkhhint}
also in any other inertial or non-inertial cartesian coordinate
system $(**)$ even in the case when \er{slochangGGaaffgfgjhjgghint}
is not valid in the system $(**)$, provided that under the change of
coordinate system we have $A'=A$ and $S'=S$. Furthermore, note that
although we established Eikonal equation
\er{MaxVacFullPPNmmmffffffiuiuhjuughbghhiuijghghghhjhjhhghyuyjjjhhjhjffint}
in the delicate Geometric Optics approximation
\er{slochangGGaaffgfgjhjgghint}, since by
\er{MaxVacFullPPNmmmffffffiuiuhjuughbghhuiiujjint} we have
\begin{equation}\label{MaxVacFullPPNmmmffffffiuiuhjuughbghhuiiujjjjjjjjgjjgint}
U(\vec x,t)=A(\vec x,t)e^{ik_0S(\vec x,t)}\,,
\end{equation}
then, considering the approximations of $A$ and $S$ as in
\er{MaxVacFullPPNmmmffffffiuiuhjuughbghhiuijghghghhhfhhghghguygtjjjkkjjkkhhint},
\er{MaxVacFullPPNmmmffffffiuiuhjuughbghhiuijghghghhjhjhhghyuyjjjhhjhjffint}
and putting them into
\er{MaxVacFullPPNmmmffffffiuiuhjuughbghhuiiujjjjjjjjgjjgint} we
actually establish $U(\vec x,t)$ only in the \underline{rough}
Geometric Optics approximation  \er{slochangGGaaffgfgjhjggh11int}.
Finally, denoting
\begin{equation}\label{MaxVacFullPPNmmmffffffiuiuhjuughbghhiuijghghghhhgjgfghjdfjgddg113399int}
\vec h(\vec x,t):=\vec {\tilde u}(\vec x,t)-c^2_0(\vec
x,t)\left(\frac{\partial S}{\partial t}(\vec x,t)+\vec {\tilde
u}(\vec x,t)\cdot\nabla_{\vec x}S(\vec x,t)\right)^{-1}\nabla_{\vec
x}S(\vec x,t),
\end{equation}
and
\begin{multline}\label{MaxVacFullPPNmmmffffffiuiuhjuughbghhiuijghghghhhfhhghghguygtjuuujjjkhjhjh11kklkloiouu3399bbint}
G\left(\vec x,t\right):=\\
\left(\frac{c^2_0}{2\left(\frac{\partial S}{\partial t}+\vec {\tilde
u}\cdot\nabla_{\vec x}S\right)}\left(\Delta_{\vec
x}S-\frac{\partial}{\partial
t}\left\{\frac{1}{c^2_0}\left(\frac{\partial S}{\partial t}+\vec
{\tilde u}\cdot\nabla_{\vec x}S\right)\right\}-\Div_{\vec x}
\left\{\frac{1}{c^2_0}\left(\frac{\partial S}{\partial t}+\vec
{\tilde u}\cdot\nabla_{\vec x} S\right)\vec {\tilde
u}\right\}\right)\right) ,
\end{multline}
we rewrite
\er{MaxVacFullPPNmmmffffffiuiuhjuughbghhiuijghghghhhfhhghghguygtjjjkkjjkkhhint}
and
\er{MaxVacFullPPNmmmffffffiuiuhjuughbghhiuijghghghhjhjhhghyuyjjjhhjhjffint}
as:
\begin{equation}\label{MaxVacFullPPNmmmffffffiuiuhjuughbghhiuijghghghhhfhhghghguygtjjjkkjjkkhhklklkjk33hhjjjljint}
\frac{\partial A}{\partial t}(\vec x,t)
+\vec h(\vec x,t)\cdot\nabla_{\vec x}A(\vec x,t)
-G(\vec x,t) A(\vec x,t)=0\,,
\end{equation}
and
\begin{equation}\label{MaxVacFullPPNmmmffffffiuiuhjuughbghhiuijghghghhhgjgfghjdfjgddg1133khhujuuuo99klklljljljint}
\left|\vec h(\vec x,t)-\vec {\tilde u}(\vec x,t)\right|^2=c^2_0(\vec
x,t)\,.
\end{equation}

Next consider the curve $\vec r(t):[t_0,b]\to\mathbb{R}^3$,
parameterized by the time variable $t$, defined as a solution of
ordinary differential equation:
\begin{equation}\label{MaxVacFullPPNmmmffffffiuiuhjuughbghhiuijghghghhhfhhghghguygtjuuujjjkk113int}
\begin{cases}
\frac{d\vec r}{dt}(t)=\vec h\left(\vec r(t),t\right)\\
\vec r(t_0)=\vec x_0,
\end{cases}
\end{equation}
where $\vec h$ was defined in
\er{MaxVacFullPPNmmmffffffiuiuhjuughbghhiuijghghghhhgjgfghjdfjgddg113399int},
then, by
\er{MaxVacFullPPNmmmffffffiuiuhjuughbghhiuijghghghhhfhhghghguygtjjjkkjjkkhhklklkjk33hhjjjljint}
and the Chain rule we have:
\begin{equation}\label{MaxVacFullPPNmmmffffffiuiuhjuughbghhiuijghghghhhfhhghghguygtjuuujjjkhjhjh11nint}
\frac{d}{dt}\left(A\left(\vec r(t),t\right)\right)=\nabla_{\vec
x}A\left(\vec r(t),t\right)\cdot\frac{d\vec r}{dt}(t)+\frac{\partial
A}{\partial t}\left(\vec r(t),t\right)= G\left(\vec
r(t),t\right)A\left(\vec r(t),t\right),
\end{equation}
where $G$ was defined in
\er{MaxVacFullPPNmmmffffffiuiuhjuughbghhiuijghghghhhfhhghghguygtjuuujjjkhjhjh11kklkloiouu3399bbint}.
Then
\er{MaxVacFullPPNmmmffffffiuiuhjuughbghhiuijghghghhhfhhghghguygtjuuujjjkhjhjh11nint}
implies
\begin{equation}\label{MaxVacFullPPNmmmffffffiuiuhjuughbghhiuijghghghhhfhhghghguygtjuuujjjkhjhjhffgg11int}
A\left(\vec r(t),t\right)=A\left(\vec
x_0,t_0\right)e^{\int_{t_0}^{t}G\left(\vec
r(\tau),\tau\right)d\tau}\quad\quad\forall t\in[t_0,b].
\end{equation}
In particular,
\begin{multline}\label{MaxVacFullPPNmmmffffffiuiuhjuughbghhiuijghghghhhfhhghghguygtjuuujjjkhjhjhffggiouuiiu113int}
A\left(\vec x_0,t_0\right)=0\;\;\text{implies}\;\; A\left(\vec
r(t),t\right)=0\quad\forall t\in[t_0,b],\\ \quad\text{and}\quad
A\left(\vec x_0,t_0\right)\neq 0\;\;\text{implies}\;\; A\left(\vec
r(t),t\right)\neq 0\quad\forall t\in[t_0,b].
\end{multline}
Therefore, by
\er{MaxVacFullPPNmmmffffffiuiuhjuughbghhiuijghghghhhfhhghghguygtjuuujjjkhjhjhffggiouuiiu113int}
we deduce that the curve that satisfies
\er{MaxVacFullPPNmmmffffffiuiuhjuughbghhiuijghghghhhfhhghghguygtjuuujjjkk113int}
coincides with the ray of light that passes through the point $\vec
x_0$ at the instant of time $t_0$. So, equality
\er{MaxVacFullPPNmmmffffffiuiuhjuughbghhiuijghghghhhfhhghghguygtjuuujjjkk113int}
is the equation of a ray and the vector field $\vec h$ defined for
every $\vec x$ by
\er{MaxVacFullPPNmmmffffffiuiuhjuughbghhiuijghghghhhgjgfghjdfjgddg113399int}
is the direction of the propagation of the ray that passes through
point $\vec x$ at the instant of time $t$. On the other hand, as
before, we can easily prove that the vector field defined in every
inertial or non-inertial coordinate system by
\er{MaxVacFullPPNmmmffffffiuiuhjuughbghhiuijghghghhhgjgfghjdfjgddg113399int}
is a \underline{speed-like} vector field. Moreover, by
\er{MaxVacFullPPNmmmffffffiuiuhjuughbghhiuijghghghhhgjgfghjdfjgddg1133khhujuuuo99klklljljljint}
the following implication holds:
\begin{equation}\label{MaxVacFullPPNmmmffffffiuiuhjuughbghhiuijghghghhhfhhghghguygtjuuujjjkhjhjhffggiouuiiu11mljjjkk;kint}
\vec {\tilde u}=0\quad\text{implies}\quad|\vec h|=c_0.
\end{equation}
Finally, by chain rule, for the curve that satisfies
\er{MaxVacFullPPNmmmffffffiuiuhjuughbghhiuijghghghhhfhhghghguygtjuuujjjkk11}
we have:
\begin{equation}\label{MaxVacFullPPNmmmffffffiuiuhjuughbghhiuijghghghhhfhhghghguygtjuuujjjkhjhjhffggiouuiiu11mljjjkk;kljjkljljhjint}
\frac{d}{dt}\left(S\left(\vec r(t),t\right)\right)=\frac{\partial
S}{\partial t}\left(\vec r(t),t\right)+\nabla_{\vec x}S\left(\vec
r(t),t\right)\cdot\frac{d\vec r}{dt}(t)=\frac{\partial S}{\partial
t}\left(\vec r(t),t\right)+\nabla_{\vec x}S\left(\vec
r(t),t\right)\cdot\vec h\left(\vec r(t),t\right).
\end{equation}
Thus, by
\er{MaxVacFullPPNmmmffffffiuiuhjuughbghhiuijghghghhhgjgfghjdfjgddg113399int},
\er{MaxVacFullPPNmmmffffffiuiuhjuughbghhiuijghghghhhgjgfghjdfjgddg1133khhujuuuo99klklljljljint}
and
\er{MaxVacFullPPNmmmffffffiuiuhjuughbghhiuijghghghhhfhhghghguygtjuuujjjkhjhjhffggiouuiiu11mljjjkk;kljjkljljhjint}
we deduce
\begin{equation}\label{MaxVacFullPPNmmmffffffiuiuhjuughbghhiuijghghghhhfhhghghguygtjuuujjjkhjhjhffggiouuiiu11mljjjkk;kljjkljljhjklklklklint}
\frac{d}{dt}\left(S\left(\vec
r(t),t\right)\right)=0\quad\quad\forall t\geq t_0,
\end{equation}
and so
\begin{equation}\label{MaxVacFullPPNmmmffffffiuiuhjuughbghhiuijghghghhhfhhghghguygtjuuujjjkhjhjhffggiouuiiu11mljjjkk;kljjkljljhjklklklkliioint}
S\left(\vec r(t),t\right)=S\left(\vec
x_0,t_0\right)\quad\quad\forall t\geq t_0.
\end{equation}
By all these facts, vector field $\vec h(\vec x,t)$ defined by
\er{MaxVacFullPPNmmmffffffiuiuhjuughbghhiuijghghghhhgjgfghjdfjgddg113399}
can be considered as the vector of the velocity (speed) of the wave
at the point $\vec x$ at the instant of time $t$. Moreover, by
\er{MaxVacFullPPNmmmffffffiuiuhjuughbghhiuijghghghhhgjgfghjdfjgddg1133khhujuuuo99klklljljljint}
we have
\begin{equation}\label{MaxVacFullPPNmmmffffffiuiuhjuughbghhiuijghghghhjhjhhghyuyjjjhhjhjffkklkljljjkhhjkjhghghghl;kl;;k;kkljjkjkll;l;int}
\left|\vec h-\vec{\tilde u}\right|^2=c^2_0.
\end{equation}
\begin{remark}\label{hjgjgghjj11nnjjint}
In contrast to the proof of
\er{MaxVacFullPPNmmmffffffiuiuhjuughbghhiuijghghghhjhjhhghyuyjjjhhjhjffint},
we do not use any of the Geometric Optics approximations
\er{slochangGGaaffgfgjhjgghint} or \er{slochangGGaaffgfgjhjggh11int}
in the proof of
\er{MaxVacFullPPNmmmffffffiuiuhjuughbghhiuijghghghhhfhhghghguygtjjjkkjjkkhhint}
and
\er{MaxVacFullPPNmmmffffffiuiuhjuughbghhiuijghghghhhfhhghghguygtjjjkkjjkkhhklklkjk33hhjjjljint}.
So,
\er{MaxVacFullPPNmmmffffffiuiuhjuughbghhiuijghghghhhfhhghghguygtjjjkkjjkkhhklklkjk33hhjjjljint}
is still valid without assumptions of the Geometric Optics
approximation \er{slochangGGaaffgfgjhjgghint} or
\er{slochangGGaaffgfgjhjggh11int} and thus, the vector field $\vec
h$, defined by
\er{MaxVacFullPPNmmmffffffiuiuhjuughbghhiuijghghghhhgjgfghjdfjgddg113399int},
is the direction of the propagation of the ray that passes through
point $\vec x$ also in the general case. However, without
assumptions of the Geometric Optics approximation we cannot derive
\er{MaxVacFullPPNmmmffffffiuiuhjuughbghhiuijghghghhjhjhhghyuyjjjhhjhjffint},
\er{MaxVacFullPPNmmmffffffiuiuhjuughbghhiuijghghghhhfhhghghguygtjuuujjjkhjhjhffggiouuiiu11mljjjkk;kint},
\er{MaxVacFullPPNmmmffffffiuiuhjuughbghhiuijghghghhhfhhghghguygtjuuujjjkhjhjhffggiouuiiu11mljjjkk;kljjkljljhjklklklkliioint}
and
\er{MaxVacFullPPNmmmffffffiuiuhjuughbghhiuijghghghhjhjhhghyuyjjjhhjhjffkklkljljjkhhjkjhghghghl;kl;;k;kkljjkjkll;l;int}
anymore.
\end{remark}

\subsubsection{The case of the monochromatic wave}
Next, up to the end of this subsection, consider the case of
monochromatic wave of the constant frequency
$\nu=\frac{\omega}{2\pi}$ where the fields $\vec {\tilde u}$ and
$c_0$ are independent of the time variable i.e. the case of
\er{MaxVacFullPPNmmmffffffiuiuhjuughbghhuiiujjjjjjjjint} where we
have
\begin{equation}\label{MaxVacFullPPNmmmffffffiuiuhjuughbghhiuijghghghhjhjhhghyuyhhjhjihikjjjjkjljklint}
\begin{cases}
\frac{\partial T}{\partial t}=\omega\\
\frac{\partial A}{\partial t}=0\\
\frac{\partial \vec {\tilde u}}{\partial t}=0\\
\frac{\partial c_0}{\partial t}=0.
\end{cases}
\end{equation}
We also assume that either our medium has no dispersion (i.e. in the
case of an electromagnetic wave \er{EBDHTrans444knbuihuigGGint} is
valid), or the velocity of the medium is negligible i.e. $\vec
u\equiv 0$ (and so in the case of an electromagnetic wave we have
$\vec{\tilde u}\equiv \vec v$) and our medium is transparent for the
given frequency $\nu=\frac{\omega}{2\pi}$. Then, by
\er{MaxVacFullPPNmmmffffffiuiuhjuughbghhuiiujjjjjjjjhhhjjjint} and
\er{MaxVacFullPPNmmmffffffiuiuhjuughbghhuiiujjjjjjjjhhhjjjkkint} we
rewrite
\er{MaxVacFullPPNmmmffffffiuiuhjuughbghhiuijghghghhjhjhhghyuyhhjhjihikjjjjkjljklint}
as
\begin{equation}\label{MaxVacFullPPNmmmffffffiuiuhjuughbghhiuijghghghhjhjhhghyuyhhjhjihikjjjint}
\begin{cases}
\frac{\partial S}{\partial t}=c\\
\frac{\partial A}{\partial t}=0.
\end{cases}
\end{equation}
Thus $\nabla_{\vec x}S$ is independent of $t$ and moreover, by
\er{MaxVacFullPPNmmmffffffiuiuhjuughbghhiuijghghghhjhjhhghyuyhhjhjihikjjjint}
we rewrite
\er{MaxVacFullPPNmmmffffffiuiuhjuughbghhiuijghghghhjhjhhghyuyjjjhhjhjffint}
as:
\begin{equation}\label{MaxVacFullPPNmmmffffffiuiuhjuughbghhiuijghghghhjhjhhghyuyiyyujjint}
\frac{c^2}{c^2_0}\left(1+\frac{1}{c}\vec {\tilde u}\cdot\nabla_\vec
x S\right)^2=\left|\nabla_\vec x S\right|^2,
\end{equation}
and, using
\er{MaxVacFullPPNmmmffffffiuiuhjuughbghhiuijghghghhjhjhhghyuyhhjhjihikjjjint}
we rewrite
\er{MaxVacFullPPNmmmffffffiuiuhjuughbghhiuijghghghhhfhhghghguygtjjjkkjjkkhhint}
as:
\begin{equation}\label{MaxVacFullPPNmmmffffffiuiuhjuughbghhiuijghghghhhfhhghghguygtjuuujjint}
2\left(\nabla_{\vec x}S-\frac{c}{c_0}\left(1+\frac{1}{c}\left(\vec
{\tilde u}\cdot\nabla_{\vec x}S\right)\right)\frac{\vec {\tilde
u}}{c_0}\right)\cdot\nabla_{\vec x}A=\left(\Div_{\vec
x}\left\{\frac{1}{c^2_0}\left(c+\vec {\tilde u}\cdot\nabla_{\vec
x}S\right)\vec {\tilde u}\right\}-\left(\Delta_{\vec
x}S\right)\right)A.
\end{equation}

In particular, in the case of the region of the space where the
following approximation is valid:
\begin{equation}\label{ojhkkint}
\frac{|\vec {\tilde u}|^2}{c^2_0}\ll 1,
\end{equation}
up to order $O\left(\frac{|\vec {\tilde u}|^2}{c^2_0}\right)$, we
rewrite
\er{MaxVacFullPPNmmmffffffiuiuhjuughbghhiuijghghghhjhjhhghyuyiyyujjint}
as:
\begin{equation}\label{MaxVacFullPPNmmmffffffiuiuhjuughbghhiuijghghghhjhjhhghyuyiyyujjljkint}
\left|\frac{c\vec {\tilde u}}{c^2_0}-\nabla_\vec x
S\right|^2=\frac{c^2}{c^2_0},
\end{equation}
and
\er{MaxVacFullPPNmmmffffffiuiuhjuughbghhiuijghghghhhfhhghghguygtjuuujjint}
as:
\begin{equation}\label{MaxVacFullPPNmmmffffffiuiuhjuughbghhiuijghghghhhfhhghghguygtjuuujjjkint}
\left(\frac{c\vec {\tilde u}}{c^2_0}-\nabla_{\vec
x}S\right)\cdot\nabla_{\vec x}A+\frac{1}{2}\left(\Div_{\vec
x}\left\{\frac{c}{c^2_0}\vec {\tilde u}\right\}-\Delta_{\vec
x}S\right)A=0.
\end{equation}
The Eikonal equation
\er{MaxVacFullPPNmmmffffffiuiuhjuughbghhiuijghghghhjhjhhghyuyiyyujjljkint}
and equation of the beam propagation
\er{MaxVacFullPPNmmmffffffiuiuhjuughbghhiuijghghghhhfhhghghguygtjuuujjjkint}
are two basic equations of propagation of monochromatic light in the
Geometric Optics approximation inside a moving medium or/and in the
presence of non-trivial gravitational field, provided that the field
$\vec {\tilde u}$ satisfies \er{ojhkkint}.

Next if we consider an arbitrary characteristic curve $\vec
r(s):[a,b]\to\mathbb{R}^3$ of equation
\er{MaxVacFullPPNmmmffffffiuiuhjuughbghhiuijghghghhhfhhghghguygtjuuujjjkint}
defined as a solution of ordinary differential equation
\begin{equation}\label{MaxVacFullPPNmmmffffffiuiuhjuughbghhiuijghghghhhfhhghghguygtjuuujjjkkint}
\begin{cases}
\frac{d\vec r}{ds}(s)=\frac{c}{c^2_0\left(\vec r(s)\right)}\vec
{\tilde u}\left(\vec
r(s)\right)-\nabla_{\vec x}S\left(\vec r(s)\right)\\
\vec r(a)=\vec x_0,
\end{cases}
\end{equation}
then, as before, by
\er{MaxVacFullPPNmmmffffffiuiuhjuughbghhiuijghghghhhfhhghghguygtjuuujjjkint}
and
\er{MaxVacFullPPNmmmffffffiuiuhjuughbghhiuijghghghhhfhhghghguygtjuuujjjkkint}
we have
\begin{equation}\label{MaxVacFullPPNmmmffffffiuiuhjuughbghhiuijghghghhhfhhghghguygtjuuujjjkhjhjhint}
\frac{d}{ds}\left(A\left(\vec r(s)\right)\right)=\nabla_{\vec
x}A\left(\vec r(s)\right)\cdot\frac{d\vec
r}{ds}(s)=\frac{1}{2}\left(\Delta_{\vec x}S\left(\vec
r(s)\right)-\Div_{\vec x}\left\{\frac{c}{c^2_0\left(\vec
r(s)\right)}\vec {\tilde u}\left(\vec
r(s)\right)\right\}\right)A\left(\vec r(s)\right),
\end{equation}
that implies
\begin{equation}\label{MaxVacFullPPNmmmffffffiuiuhjuughbghhiuijghghghhhfhhghghguygtjuuujjjkhjhjhffggint}
A\left(\vec r(s)\right)=A\left(\vec
x_0\right)e^{\frac{1}{2}\int_a^{s}\left(\Delta_{\vec x}S\left(\vec
r(\tau)\right)-\Div_{\vec x}\left\{\frac{c}{c^2_0\left(\vec
r(\tau)\right)}\vec {\tilde u}\left(\vec
r(\tau)\right)\right\}\right)d\tau}\quad\quad\forall s\in[a,b].
\end{equation}
In particular,
\begin{equation}\label{MaxVacFullPPNmmmffffffiuiuhjuughbghhiuijghghghhhfhhghghguygtjuuujjjkhjhjhffggiouuiiuint}
A\left(\vec x_0\right)=0\;\;\text{implies}\;\; A\left(\vec
r(s)\right)=0\quad\forall s\in[a,b],\quad\text{and}\quad A\left(\vec
x_0\right)\neq 0\;\;\text{implies}\;\; A\left(\vec r(s)\right)\neq
0\quad\forall s\in[a,b].
\end{equation}
Therefore, by
\er{MaxVacFullPPNmmmffffffiuiuhjuughbghhiuijghghghhhfhhghghguygtjuuujjjkhjhjhffggiouuiiuint}
we deduce that in the case of \er{ojhkkint} the curve that satisfies
\er{MaxVacFullPPNmmmffffffiuiuhjuughbghhiuijghghghhhfhhghghguygtjuuujjjkkint}
coincides with the ray of light that passes through the point $\vec
x_0$. So in the case of \er{ojhkkint}, equality
\er{MaxVacFullPPNmmmffffffiuiuhjuughbghhiuijghghghhhfhhghghguygtjuuujjjkkint}
is the equation of a ray and the vector field $\vec h_0$ defined for
every $\vec x$ by:
\begin{equation}\label{MaxVacFullPPNmmmffffffiuiuhjuughbghhiuijghghghhhfhhghghguygtjuuujjjkyuuyint}
\vec h_0(\vec x):=\frac{c}{c^2_0(\vec x)}\vec {\tilde u}(\vec
x)-\nabla_{\vec x}S(\vec x)\approx\frac{c}{c^2_0(\vec x)}\vec h(\vec
x),
\end{equation}
is the direction of the propagation of the ray that passes through
point $\vec x$. Moreover, by
\er{MaxVacFullPPNmmmffffffiuiuhjuughbghhiuijghghghhjhjhhghyuyiyyujjljkint}
$\vec h_0$ satisfies
\begin{equation}\label{MaxVacFullPPNmmmffffffiuiuhjuughbghhiuijghghghhjhjhhghyuyiyyujjljkgghhgint}
|\vec h_0|^2=\frac{c^2}{c^2_0}.
\end{equation}
\begin{remark}\label{hjgjgghint}
In contrast to the proof of
\er{MaxVacFullPPNmmmffffffiuiuhjuughbghhiuijghghghhjhjhhghyuyjjjhhjhjffint},
\er{MaxVacFullPPNmmmffffffiuiuhjuughbghhiuijghghghhjhjhhghyuyiyyujjint}
or
\er{MaxVacFullPPNmmmffffffiuiuhjuughbghhiuijghghghhjhjhhghyuyiyyujjljkint},
we do not use the Geometric Optics approximations
\er{slochangGGaaffgfgjhjgghint} or \er{slochangGGaaffgfgjhjggh11int}
in the proof of
\er{MaxVacFullPPNmmmffffffiuiuhjuughbghhiuijghghghhhfhjjint},
\er{MaxVacFullPPNmmmffffffiuiuhjuughbghhiuijghghghhhfhhghghguygtjjjkkjjkkhhint},
\er{MaxVacFullPPNmmmffffffiuiuhjuughbghhiuijghghghhhfhhghghguygtjuuujjint}
and
\er{MaxVacFullPPNmmmffffffiuiuhjuughbghhiuijghghghhhfhhghghguygtjuuujjjkint}.
We just need the estimation \er{ojhkkint} for the proof of
\er{MaxVacFullPPNmmmffffffiuiuhjuughbghhiuijghghghhhfhhghghguygtjuuujjjkint}.
So,
\er{MaxVacFullPPNmmmffffffiuiuhjuughbghhiuijghghghhhfhhghghguygtjuuujjjkhjhjhffggiouuiiuint}
is still valid without assumptions of the Geometric Optics
approximation \er{slochangGGaaffgfgjhjgghint} or
\er{slochangGGaaffgfgjhjggh11int} and thus, the vector field $\vec
h_0$, defined by
\er{MaxVacFullPPNmmmffffffiuiuhjuughbghhiuijghghghhhfhhghghguygtjuuujjjkyuuyint},
is the direction of the propagation of the ray that passes through
point $\vec x$ also in the general case, provided the estimation
\er{ojhkkint} holds. However, without assumption of the Geometric
Optics approximation we cannot derive
\er{MaxVacFullPPNmmmffffffiuiuhjuughbghhiuijghghghhjhjhhghyuyiyyujjljkgghhgint}
anymore.
\end{remark}
Next by
\er{MaxVacFullPPNmmmffffffiuiuhjuughbghhiuijghghghhjhjhhghyuyiyyujjljkint}
and
\er{MaxVacFullPPNmmmffffffiuiuhjuughbghhiuijghghghhhfhhghghguygtjuuujjjkint}
in subsection \ref{GO} we prove the following Fermat Principle:
\begin{proposition}\label{gughghfint}
Assume Geometric Optics approximation together with \er{ojhkkint}.
Then the light that travels from point $N$ to point $M$ chooses the
path $\vec r(s):[a,b]\to\mathbb{R}^3$ with endpoints $\vec r(a)=N$
and $\vec r(b)=M$ which minimizes the quantity:
\begin{equation}\label{MaxMedFullGGffgggyyojjhhjkhjyyiuhggjhhjhuyytytyuuytrrtghjtyuggyuighjuyioyyfgffhyuhhghzzrrkkhhkkkhhhjhkjhhghhgghint}
J\left(\vec r(\cdot)\right):=\int_a^bn\left(\vec
r(s)\right)\left|\vec r'(s)\right|ds-\int_a^b
\frac{1}{c}n^2\left(\vec r(s)\right)\vec {\tilde u}\left(\vec
r(s)\right)\cdot\vec r'(s)ds,
\end{equation}
where we set the refraction index:
\begin{equation}\label{MaxMedFullGGffgggyyojjhhjkhjyyiuhggjhhjhuyytytyuuytrrtghjtyuggyuighjuyioyyfgffhyuhhghzzrrkkhhkkkhhhjhkjhhghhgghiuiu1int}
n\left(\vec x\right):=\frac{c}{c_0\left(\vec x\right)}.
\end{equation}
Moreover, if the path $\vec r(s):[a,b]\to\mathbb{R}^3$ with
endpoints $\vec r(a)=N$ and $\vec r(b)=M$ is the real path of the
light, then:
\begin{equation}\label{MaxMedFullGGffgggyyojjhhjkhjyyiuhggjhhjhuyytytyuuytrrtghjtyuggyuighjuyioyyfgffhyuhhghzzrrkkhhkkkint}
\left(-S(M)\right)- \left(-S(N)\right)=\int_a^bn\left(\vec
r(s)\right)\left|\vec r'(s)\right|ds-\int_a^b
\frac{1}{c}n^2\left(\vec r(s)\right)\vec {\tilde u}\left(\vec
r(s)\right)\cdot\vec r'(s)ds.
\end{equation}
\end{proposition}
See also subsection \ref{GO} for the generalization of the Fermat
Principle to the case where we cannot take \er{ojhkkint} into
account.

In particular, by Proposition \ref{gughghfint} the path of travel of
the light satisfies the Euler-Lagrange equation for the functional
$J\left(\vec r(\cdot)\right)$,
that is the differential equation of the path of light:
\begin{multline}\label{MaxMedFullGGffgggyyojjhhjkhjyyiuhggjhhjhuyytytyuuytrrtghjtyuggyuighjuyioyyfgffhyuhhghzzrrkkhhkkkhhhjhkjhhghhgghiuiuhhjhjint}
\frac{d}{d\lambda}\left(n\left(\vec r\right)\frac{d\vec
r}{d\lambda}\right)=\frac{1}{c}n^2\left(\vec
r\right)\left(curl_{\vec x}\vec {\tilde u}\left(\vec
r\right)\right)\times\frac{d\vec r}{d\lambda}\\ +\nabla_{\vec
x}n\left(\vec r\right) +\frac{2}{c}n\left(\vec r\right)\left\{\vec
{\tilde u}\left(\vec r\right)\otimes\nabla_{\vec x}n\left(\vec
r\right)-\nabla_{\vec x}n\left(\vec r\right)\otimes\vec {\tilde
u}\left(\vec r\right)\right\}\cdot\frac{d\vec r}{d\lambda},
\end{multline}
where
\begin{equation}\label{MaxMedFullGGffgggyyojjhhjkhjyyiuhggjhhjhuyytytyuuytrrtghjtyuggyuighjuyioyyfgffhyuhhghzzrrkkhhkkkhhhjhkjhhghhgghiuiuint}
\lambda:=\int_a^s\left|\vec r'(\tau)\right|d\tau,
\end{equation}
is the natural parameter of the curve (see subsection \ref{GO} for
details).

Next, assume that the wave we consider has an electromagnetic
nature. Then by \er{gughhghfbvnbvint} and \er{uyuyuyyint} we have
\begin{equation}\label{gughhghfbvnbvyyuuyrint}
c_0=c\sqrt{\kappa_0\gamma_0}\quad\text{and}\quad\vec {\tilde
u}=\left(\gamma_0\kappa_0\vec v+(1-\gamma_0\kappa_0)\vec u\right),
\end{equation}
where, $\vec u$ is the medium velocity and $\vec v$ is the local
vectorial gravitational potential. Moreover, assume that we consider
light traveling in some region either filled with the resting medium
of constant dielectric permeability $\gamma_0$ and magnetic
permeability $\kappa_0$ or in the vacuum. Then by
\er{gughhghfbvnbvyyuuyrint} and
\er{MaxMedFullGGffgggyyojjhhjkhjyyiuhggjhhjhuyytytyuuytrrtghjtyuggyuighjuyioyyfgffhyuhhghzzrrkkhhkkkhhhjhkjhhghhgghiuiu1int}
we have:
\begin{equation}\label{gughhghfbvnbvyyuuyrhiyyuiuuint}
n=\frac{1}{\sqrt{\kappa_0\gamma_0}}\;\;\;\text{is a
constant,}\quad\text{and}\quad\vec {\tilde u}=\gamma_0\kappa_0\vec
v,
\end{equation}
Then by \er{gughhghfbvnbvyyuuyrhiyyuiuuint} we rewrite
\er{MaxMedFullGGffgggyyojjhhjkhjyyiuhggjhhjhuyytytyuuytrrtghjtyuggyuighjuyioyyfgffhyuhhghzzrrkkhhkkkhhhjhkjhhghhgghiuiuhhjhjint}
as:
\begin{equation}\label{MaxMedFullGGffgggyyojjhhjkhjyyiuhggjhhjhuyytytyuuytrrtghjtyuggyuighjuyioyyfgffhyuhhghzzrrkkhhkkkhhhjhkjhhghhgghiuiuhhjhjiuiuyuint}
\frac{d^2\vec
r}{d\lambda^2}=\frac{1}{c}\sqrt{\gamma_0\kappa_0}\,\left(curl_{\vec
x}\vec v\left(\vec r\right)\right)\times\frac{d\vec r}{d\lambda}.
\end{equation}
In particular, if our coordinate system is inertial, or more
generally non-rotating, then $curl_{\vec x}\vec v=0$ and we deduce
that the path of the light from the point $N$ to the point $M$ is
the direct line connecting these points, provided we take in the
account estimation \er{ojhkkint}.

On the other hand, if our system is rotating, then, since $\vec v$
is a speed-like vector field, we clearly deduce:
\begin{equation}\label{MaxMedFullGGffgggyyojjhhjkhjyyiuhggjhhjhuyytytyuuytrrtghjtyuggyuighjuyioyyfgffhyuhhghzzrrkkhhkkkhhhjhkjhhghhgghiuiuhhjhjiuiuyujjkint}
curl_{\vec x}\vec v=-2\vec w,
\end{equation}
where $\vec w$ is the vector of the angular speed of rotation of our
coordinate system. Thus by inserting
\er{MaxMedFullGGffgggyyojjhhjkhjyyiuhggjhhjhuyytytyuuytrrtghjtyuggyuighjuyioyyfgffhyuhhghzzrrkkhhkkkhhhjhkjhhghhgghiuiuhhjhjiuiuyujjkint}
into
\er{MaxMedFullGGffgggyyojjhhjkhjyyiuhggjhhjhuyytytyuuytrrtghjtyuggyuighjuyioyyfgffhyuhhghzzrrkkhhkkkhhhjhkjhhghhgghiuiuhhjhjiuiuyuint}
we deduce:
\begin{equation}\label{MaxMedFullGGffgggyyojjhhjkhjyyiuhggjhhjhuyytytyuuytrrtghjtyuggyuighjuyioyyfgffhyuhhghzzrrkkhhkkkhhhjhkjhhghhgghiuiuhhjhjiuiuyuyyuint}
\frac{d^2\vec
r}{d\lambda^2}=-\frac{2}{c}\sqrt{\gamma_0\kappa_0}\,\vec
w\times\frac{d\vec r}{d\lambda}.
\end{equation}
%
%
%
%
%
%
The solution of
\er{MaxMedFullGGffgggyyojjhhjkhjyyiuhggjhhjhuyytytyuuytrrtghjtyuggyuighjuyioyyfgffhyuhhghzzrrkkhhkkkhhhjhkjhhghhgghiuiuhhjhjiuiuyuyyuint}
is the following:
\begin{equation}\label{MaxMedFullGGffgggyyojjhhjkhjyyiuhggjhhjhuyytytyuuytrrtghjtyuggyuighjuyioyyfgffhyuhhghzzrrkkhhkkkhhhjhkjhhghhgghiuiuhhjhjiuiuyuyyukklghhgint}
\begin{cases}
x(\lambda)=C_1\frac{c}{2w}\sqrt{\kappa_0\gamma_0}\,\left(\cos{\left(\frac{2w}{c}\sqrt{\kappa_0\gamma_0}\,\lambda\right)}-1\right)+C_2\frac{c}{2w}\sqrt{\kappa_0\gamma_0}\,\sin{\left(\frac{2w}{c}\sqrt{\kappa_0\gamma_0}\,\lambda\right)}+D_1
\\
y(\lambda)=-C_1\frac{c}{2w}\sqrt{\kappa_0\gamma_0}\,\sin{\left(\frac{2w}{c}\sqrt{\kappa_0\gamma_0}\,\lambda\right)}+C_2\frac{c}{2w}\sqrt{\kappa_0\gamma_0}\,\left(\cos{\left(\frac{2w}{c}\sqrt{\kappa_0\gamma_0}\,\lambda\right)}-1\right)+D_2
\\
z(\lambda)=C_3\lambda+D_3,
\end{cases}
\end{equation}
where, since $\lambda$ is a natural parameter of the curve, we have:
\begin{equation}\label{MaxMedFullGGffgggyyojjhhjkhjyyiuhggjhhjhuyytytyuuytrrtghjtyuggyuighjuyioyyfgffhyuhhghzzrrkkhhkkkhhhjhkjhhghhgghiuiuhhjhjiuiuyuyyuojkint}
C^2_1+C^2_2+C^2_3=1.
\end{equation}
So, the curve in
\er{MaxMedFullGGffgggyyojjhhjkhjyyiuhggjhhjhuyytytyuuytrrtghjtyuggyuighjuyioyyfgffhyuhhghzzrrkkhhkkkhhhjhkjhhghhgghiuiuhhjhjiuiuyuyyukklghhgint}
is the trajectory of the light in the rotating coordinate system,
provided we assume \er{ojhkkint}. In particular, by
\er{MaxMedFullGGffgggyyojjhhjkhjyyiuhggjhhjhuyytytyuuytrrtghjtyuggyuighjuyioyyfgffhyuhhghzzrrkkhhkkkhhhjhkjhhghhgghiuiuhhjhjiuiuyuyyukklghhgint}
we have:
\begin{equation}\label{MaxMedFullGGffgggyyojjhhjkhjyyiuhggjhhjhuyytytyuuytrrtghjtyuggyuighjuyioyyfgffhyuhhghzzrrkkhhkkkhhhjhkjhhghhgghiuiuhhjhjiuiuyuyyuojkjkkkhguggint}
\begin{cases}
x(0)=D_1,\quad y(0)=D_2,\quad z(0)=D_3,\\
\frac{dx}{d\lambda}(0)=C_2,\quad\frac{dy}{d\lambda}(0)=-C_1,\quad\frac{dz}{d\lambda}(0)=C_3.
\end{cases}
\end{equation}
%
%
%
%
%
%
The constants $C_1,C_2,C_3,D_1,D_2,D_3$ can be determined either by
the initial data
\er{MaxMedFullGGffgggyyojjhhjkhjyyiuhggjhhjhuyytytyuuytrrtghjtyuggyuighjuyioyyfgffhyuhhghzzrrkkhhkkkhhhjhkjhhghhgghiuiuhhjhjiuiuyuyyuojkjkkkhguggint}
or by the beginning and the ending points $N$ and $M$ of the curve.

\subsubsection{The laws of reflection and refraction}
Next consider a monochromatic wave of the frequency
$\nu=\omega/(2\pi)$ characterized by:
\begin{equation}\label{MaxVacFullPPNmmmffffffiuiuhjuughbghhuiiujjhhjjhjhhjint}
U(\vec x,t)=A(\vec x)e^{ik_0S(\vec x,t)},\quad\text{where}\quad
k_0=\frac{\omega}{c}\quad\text{and}\quad\frac{\partial S}{\partial
t}=c\,,
\end{equation}
and, consistently with
\er{MaxVacFullPPNmmmffffffiuiuhjuughbghhiuijghghghhhfhhghghguygtjuuujjjkyuuyint}
consider a direction field:
\begin{equation}\label{MaxVacFullPPNmmmffffffiuiuhjuughbghhiuijghghghhhfhhghghguygtjuuujjjkyuuykjkjjint}
\vec h_0(\vec x)=\frac{c}{c^2_0(\vec x)}\vec {\tilde u}(\vec
x)-\nabla_{\vec x}S(\vec x)\approx \frac{c}{c^2_0(\vec x)}\vec
h(\vec x).
\end{equation}
Furthermore, assume that this wave undergoes reflection and/or
refraction on the stationary (time independent) surface
$\mathcal{T}$ with the outcoming unit normal $\vec n$, separating
two regions characterized respectively by $c_0=c^{(1)}_0$ and $\vec
{\tilde u}=\vec {\tilde u}_1$ and by $c^{(2)}_0$ and $\vec {\tilde
u}_2$, with the formation of the reflected wave (of the same
frequency), characterized by:
\begin{equation}\label{MaxVacFullPPNmmmffffffiuiuhjuughbghhuiiujjhhjjhjhhjughint}
U_1(\vec x,t)=A_1(\vec x)e^{ik_0S_1(\vec
x,t)},\quad\text{where}\quad\frac{\partial S_1}{\partial t}=c\,,
\end{equation}
and by a direction field:
\begin{equation}\label{MaxVacFullPPNmmmffffffiuiuhjuughbghhiuijghghghhhfhhghghguygtjuuujjjkyuuykjkjjhjhhjint}
\vec h_1(\vec x)=\frac{c}{c^2_0(\vec x)}\vec {\tilde u}(\vec
x)-\nabla_{\vec x}S_1(\vec x),
\end{equation}
and formation of the refracted wave (of the same frequency),
characterized by:
\begin{equation}\label{MaxVacFullPPNmmmffffffiuiuhjuughbghhuiiujjhhjjhjhhj1int}
U_2(\vec x,t)=A_2(\vec x)e^{ik_0S_2(\vec
x,t)},\quad\text{where}\quad\frac{\partial S_2}{\partial t}=c\,.
\end{equation}
and by a direction field:
\begin{equation}\label{MaxVacFullPPNmmmffffffiuiuhjuughbghhiuijghghghhhfhhghghguygtjuuujjjkyuuykjkjj1int}
\vec h_2(\vec x)=\frac{c}{\left(c^{(2)}_0(\vec x)\right)^2}\vec
{\tilde u_2}(\vec x)-\nabla_{\vec x}S_2(\vec x).
\end{equation}
Then the boundary conditions of $U$, $U_1$ and $U_2$ depend on the
physical meaning of these fields. However, one of the
\underline{necessary} conditions should be that
\begin{equation}\label{MaxMedFullGGffgggyyojjyugggjhhjiiuuiyuyuyyuyuzzhhkkkint}
S_1(\vec x,t)=S_2(\vec x,t)+C_2=S(\vec x,t)\quad\quad\forall\vec
x\in\mathcal{T},
\end{equation}
where $C_2$ is a real constant. In particular
\er{MaxMedFullGGffgggyyojjyugggjhhjiiuuiyuyuyyuyuzzhhkkkint}
implies:
%
%
%
\begin{equation}\label{MaxMedFullGGffgggyyojjyugggjhhjiiuuiyuyuyyuyuzzhhkkkgdfgint}
\nabla_{\vec x} S_1-\left(\vec n\cdot \nabla_{\vec x} S_1\right)\vec
n
=\nabla_{\vec x} S_2-\left(\vec n\cdot \nabla_{\vec x}
S_2\right)\vec n=\nabla_{\vec x} S-\left(\vec n\cdot \nabla_{\vec x}
S\right)\vec n\quad\quad\forall\vec x\in\mathcal{T}.
\end{equation}
In particular, for every point on the surface $\mathcal{T}$ vectors
$\nabla_{\vec x} S_1$ and $\nabla_{\vec x} S_2$ lie in the plane
formed by vectors $\vec n$ and $\nabla_{\vec x} S$. Moreover, by
\er{MaxVacFullPPNmmmffffffiuiuhjuughbghhiuijghghghhhfhhghghguygtjuuujjjkyuuykjkjjint},
\er{MaxVacFullPPNmmmffffffiuiuhjuughbghhiuijghghghhhfhhghghguygtjuuujjjkyuuykjkjjhjhhjint}
and \er{MaxMedFullGGffgggyyojjyugggjhhjiiuuiyuyuyyuyuzzhhkkkgdfgint}
we have
%
%
%
\begin{equation}\label{MaxMedFullGGffgggyyojjyugggjhhjiiuuiyuyuyyuyuzzhhkkkgdfghgghghint}
\vec h_1-\left(\vec n\cdot \vec h_1\right)\vec n=\vec h_0-\left(\vec
n\cdot \vec h_0\right)\vec n\quad\quad\forall\vec x\in\mathcal{T},
\end{equation}
and in particular, for every point on the surface $\mathcal{T}$
vector $\vec h_1$ lies in the plane formed by vectors $\vec n$ and
$\vec h_0$.
Next, assume that the approximate equations in
\er{MaxVacFullPPNmmmffffffiuiuhjuughbghhiuijghghghhjhjhhghyuyiyyujjljkint}
and
\er{MaxVacFullPPNmmmffffffiuiuhjuughbghhiuijghghghhhfhhghghguygtjuuujjjkint}
are valid in every of two regions on the both sides of
$\mathcal{T}$. Then by
\er{MaxVacFullPPNmmmffffffiuiuhjuughbghhiuijghghghhjhjhhghyuyiyyujjljkgghhgint}
we have
\begin{equation}\label{MaxVacFullPPNmmmffffffiuiuhjuughbghhiuijghghghhjhjhhghyuyiyyujjljkgghhguyuyint}
|\vec h_1|=|\vec h_0|=\frac{c}{c_0}.
\end{equation}
Then, since $\vec h_1\neq\vec h_0$, by
\er{MaxMedFullGGffgggyyojjyugggjhhjiiuuiyuyuyyuyuzzhhkkkgdfghgghghint}
and
\er{MaxVacFullPPNmmmffffffiuiuhjuughbghhiuijghghghhjhjhhghyuyiyyujjljkgghhguyuyint}
we deduce
\begin{equation}\label{MaxMedFullGGffgggyyojjyugggjhhjiiuuiyuyuyyuyuzzhhkkkgdfghgghghijhjhhkhint}
\vec n\cdot \vec h_1=-\vec n\cdot \vec h_0\quad\quad\forall\vec
x\in\mathcal{T}.
\end{equation}
So, by
\er{MaxVacFullPPNmmmffffffiuiuhjuughbghhiuijghghghhjhjhhghyuyiyyujjljkgghhguyuyint}
and
\er{MaxMedFullGGffgggyyojjyugggjhhjiiuuiyuyuyyuyuzzhhkkkgdfghgghghijhjhhkhint}
we obtain the law of reflection: vector $\vec h_1$ lies in the plane
formed by vectors $\vec n$ and $\vec h_0$, and we have:
\begin{equation}\label{MaxMedFulljhhjjjjjint}
\theta\left(\vec h_0,-\vec n\right)=\theta_1\left(\vec h_1,\vec
n\right)
\end{equation}
where $\theta\left(\vec h_0,-\vec n\right)$ is the angle between the
incoming ray direction $\vec h_0$ and the incoming normal to the
surface $-\vec n$ and $\theta_1\left(\vec h_1,\vec n\right)$ is the
angle between the reflected ray direction $\vec h_1$ and the
outcoming normal $\vec n$.

Next assume that the wave we consider in
\er{MaxVacFullPPNmmmffffffiuiuhjuughbghhuiiujjhhjjhjhhjint} has an
electromagnetic nature. Then by \er{gughhghfbvnbvyyuuyrint} we have
\begin{equation}\label{gughhghfbvnbvyyuuyint}
c_0=c\sqrt{\kappa_0\gamma_0}\quad\text{and}\quad\vec {\tilde
u}=\left(\gamma_0\kappa_0\vec v+(1-\gamma_0\kappa_0)\vec u\right),
\end{equation}
where, $\vec u$ is the medium velocity and $\vec v$ is the local
vectorial gravitational potential. Similarly, on the second side of
surface $\mathcal{T}$ we have
\begin{equation}\label{gughhghfbvnbvyyuuyGGHHGint}
c^{(2)}_0=c\sqrt{\kappa^{(2)}_0\gamma^{(2)}_0}\quad\text{and}\quad\vec
{\tilde u}_2=\left(\gamma^{(2)}_0\kappa^{(2)}_0\vec
v+(1-\gamma^{(2)}_0\kappa^{(2)}_0)\vec u^{(2)}\right),
\end{equation}
where, $\vec u^{(2)}$ is the medium velocity on the second side of
surface $\mathcal{T}$. Furthermore, assume that the medium rests on
the both sides of surface $\mathcal{T}$. Then in this particular
case we rewrite \er{gughhghfbvnbvyyuuyint} and
\er{gughhghfbvnbvyyuuyGGHHGint} as
\begin{equation}\label{gughhghfbvnbvyyuuygggint}
c_0=c\sqrt{\kappa_0\gamma_0}\quad\text{and}\quad\vec {\tilde
u}=\gamma_0\kappa_0\vec v,
\end{equation}
and
\begin{equation}\label{gughhghfbvnbvyyuuyGGHHGhjhjhint}
c^{(2)}_0=c\sqrt{\kappa^{(2)}_0\gamma^{(2)}_0}\quad\text{and}\quad\vec
{\tilde u}_2=\gamma^{(2)}_0\kappa^{(2)}_0\vec v,
\end{equation}
Then in particular, by \er{gughhghfbvnbvyyuuygggint} and
\er{gughhghfbvnbvyyuuyGGHHGhjhjhint} we deduce
\begin{equation}\label{gughhghfbvnbvyyuuyGGHHGhjhjhhhjhjint}
\frac{c}{\left(c^{(2)}_0\right)^2}\vec {\tilde
u}_2=\frac{c}{c^2_0}\vec {\tilde u}=\frac{1}{c}\,\vec v.
\end{equation}
Thus, by inserting
\er{MaxVacFullPPNmmmffffffiuiuhjuughbghhiuijghghghhhfhhghghguygtjuuujjjkyuuykjkjjint}
and \er{gughhghfbvnbvyyuuyGGHHGhjhjhhhjhjint} into
\er{MaxMedFullGGffgggyyojjyugggjhhjiiuuiyuyuyyuyuzzhhkkkgdfgint}, we
deduce:
\begin{equation}\label{MaxMedFullGGffgggyyojjyugggjhhjiiuuiyuyuyyuyuzzhhkkkgdfg1int}
\vec h_2-\left(\vec n\cdot \vec h_2\right)\vec n=\vec h_0-\left(\vec
n\cdot \vec h_0\right)\vec n\quad\quad\forall\vec x\in\mathcal{T},
\end{equation}
and in particular, for every point on the surface $\mathcal{T}$
vector $\vec h_2$ lies in the plane formed by vectors $\vec n$ and
$\vec h_0$. On the other hand by
\er{MaxVacFullPPNmmmffffffiuiuhjuughbghhiuijghghghhjhjhhghyuyiyyujjljkgghhgint}
we have:
\begin{equation}\label{MaxVacFullPPNmmmffffffiuiuhjuughbghhiuijghghghhjhjhhghyuyiyyujjljkgghhglkjlkjkjljljhhjjint}
|\vec h_0|=\frac{c}{c_0}\quad\text{and}\quad|\vec
h_2|=\frac{c}{c^{(2)}_0}.
\end{equation}
So, by
\er{MaxMedFullGGffgggyyojjyugggjhhjiiuuiyuyuyyuyuzzhhkkkgdfg1int}
and
\er{MaxVacFullPPNmmmffffffiuiuhjuughbghhiuijghghghhjhjhhghyuyiyyujjljkgghhglkjlkjkjljljhhjjint},
we have the Snell's law of refraction: vector $\vec h_2$ lies in the
plane formed by vectors $\vec n$ and $\vec h_0$, and we have:
\begin{equation}\label{MaxMedFulljhhjjjjjhhhjint}
n\sin{\left(\theta\left(\vec h_0,\vec
n\right)\right)}=n_2\sin{\left(\theta_2\left(\vec h_2,\vec
n\right)\right)}
\end{equation}
where $\theta\left(\vec h_0,\vec n\right)$ is the angle between the
incoming ray direction $\vec h_0$ and the normal to the surface
$\vec n$, $\theta_2\left(\vec h_2,\vec n\right)$ is the angle
between the refracted ray direction $\vec h_2$ and the normal $\vec
n$ and as in
\er{MaxMedFullGGffgggyyojjhhjkhjyyiuhggjhhjhuyytytyuuytrrtghjtyuggyuighjuyioyyfgffhyuhhghzzrrkkhhkkkhhhjhkjhhghhgghiuiu1int}
we set refraction indexes:
\begin{equation}\label{MaxMedFullGGffgggyyojjhhjkhjyyiuhggjhhjhuyytytyuuytrrtghjtyuggyuighjuyioyyfgffhyuhhghzzrrkkhhkkkhhhjhkjhhghhgghiuiujjkjkint}
n:=\frac{c}{c_0}\quad\text{and}\quad n_2:=\frac{c}{c^{(2)}_0}.
\end{equation}

\subsubsection{Sagnac effect}
Assume again the monochromatic electromagnetic wave of the frequency
$\nu=\omega/(2\pi)$ characterized by:
\begin{equation}\label{MaxVacFullPPNmmmffffffiuiuhjuughbghhuiiujjhhjjhjhhjhjjhhjhjjhint}
U(\vec x,t)=A(\vec x,t)e^{iT(\vec x,t)}=A(\vec x,t)e^{ik_0S(\vec
x,t)},\quad\text{where}\quad
k_0=\frac{\omega}{c}\quad\text{and}\quad\frac{\partial S}{\partial
t}=c\,.
\end{equation}
Then by \er{gughhghfbvnbvyyuuyrint} we have
\begin{equation}\label{gughhghfbvnbvyyuuyrhjhjhjint}
c_0=c\sqrt{\kappa_0\gamma_0}\quad\text{and}\quad\vec {\tilde
u}=\left(\gamma_0\kappa_0\vec v+(1-\gamma_0\kappa_0)\vec u\right),
\end{equation}
where, $\vec u$ is the medium velocity and $\vec v$ is the local
vectorial gravitational potential. Moreover, assume again that we
consider light traveling in some region either filled with the
resting medium of constant dielectric permeability $\gamma_0$ and
magnetic permeability $\kappa_0$ or in the vacuum. Then by
\er{gughhghfbvnbvyyuuyrhjhjhjint} and
\er{MaxMedFullGGffgggyyojjhhjkhjyyiuhggjhhjhuyytytyuuytrrtghjtyuggyuighjuyioyyfgffhyuhhghzzrrkkhhkkkhhhjhkjhhghhgghiuiu1int}
we have
\begin{equation}\label{gughhghfbvnbvyyuuyrhiyyuiuuhjhint}
n=\frac{1}{\sqrt{\kappa_0\gamma_0}}\;\;\;\text{is a
constant,}\quad\text{and}\quad\vec {\tilde u}=\gamma_0\kappa_0\vec
v.
\end{equation}
Next, assume that the light travels from point $N$ to point $M$
across the curve $\vec r(s):[a,b]\to\mathbb{R}^3$ with endpoints
$\vec r(a)=N$ and $\vec r(b)=M$ undergoing possibly certain number
of reflections from mirrors during its travel. Then by
\er{MaxMedFullGGffgggyyojjhhjkhjyyiuhggjhhjhuyytytyuuytrrtghjtyuggyuighjuyioyyfgffhyuhhghzzrrkkhhkkkint},
\er{gughhghfbvnbvyyuuyrhiyyuiuuhjhint} and
\er{MaxMedFullGGffgggyyojjyugggjhhjiiuuiyuyuyyuyuzzhhkkkint} we
have:
\begin{equation}\label{MaxMedFullGGffgggyyojjhhjkhjyyiuhggjhhjhuyytytyuuytrrtghjtyuggyuighjuyioyyfgffhyuhhghzzrrkkhhkkkhjjhint}
\delta(-S):=\left(-S(M^-)\right)-
\left(-S(N^+)\right)=\frac{1}{\sqrt{\kappa_0\gamma_0}}\int_a^b
\left|\vec r'(s)\right|ds-\frac{1}{c}\int_a^b\vec v\left(\vec
r(s)\right)\cdot\vec r'(s)ds.
\end{equation}
In particular, if we assume that $M=N$ i.e. our curve is closed and
moreover, our curve is the boundary of some surface $\mathcal{S}_0$,
then by Stokes Theorem we have:
\begin{multline}\label{MaxMedFullGGffgggyyojjhhjkhjyyiuhggjhhjhuyytytyuuytrrtghjtyuggyuighjuyioyyfgffhyuhhghzzrrkkhhkkkhjjhjjint}
\delta(-S)=\left(-S(M^-)\right)-
\left(-S(M^+)\right)=\frac{1}{\sqrt{\kappa_0\gamma_0}}\int_a^b
\left|\vec r'(s)\right|ds-\frac{1}{c}\iint \left(curl_{\vec x}\vec
v\right)\cdot\vec n
\,d\mathcal{S}_0\\=\frac{1}{\sqrt{\kappa_0\gamma_0}}\left|\partial
\mathcal{S}_0\right|-\frac{1}{c}\iint \left(curl_{\vec x}\vec
v\right)\cdot\vec n \,d\mathcal{S}_0,
\end{multline}
where $\vec n$ is the unit normal to the surface.
In particular, if our coordinate system is inertial, or more
generally non-rotating, then $curl_{\vec x}\vec v=0$ and by
\er{MaxMedFullGGffgggyyojjhhjkhjyyiuhggjhhjhuyytytyuuytrrtghjtyuggyuighjuyioyyfgffhyuhhghzzrrkkhhkkkhjjhjjint}
we deduce
\begin{equation}\label{MaxMedFullGGffgggyyojjhhjkhjyyiuhggjhhjhuyytytyuuytrrtghjtyuggyuighjuyioyyfgffhyuhhghzzrrkkhhkkkhjjhjjhjhhjjhjjhint}
\delta(-S)=\frac{1}{\sqrt{\kappa_0\gamma_0}}\left|\partial
\mathcal{S}_0\right|.
\end{equation}
On the other hand, if our system is rotating, then as in
\er{MaxMedFullGGffgggyyojjhhjkhjyyiuhggjhhjhuyytytyuuytrrtghjtyuggyuighjuyioyyfgffhyuhhghzzrrkkhhkkkhhhjhkjhhghhgghiuiuhhjhjiuiuyujjkint}
we clearly deduce:
\begin{equation}\label{MaxMedFullGGffgggyyojjhhjkhjyyiuhggjhhjhuyytytyuuytrrtghjtyuggyuighjuyioyyfgffhyuhhghzzrrkkhhkkkhhhjhkjhhghhgghiuiuhhjhjiuiuyujjkjjkint}
curl_{\vec x}\vec v=-2\vec w,
\end{equation}
where $\vec w$ is the vector of the angular speed of rotation of our
coordinate system. Then by
\er{MaxMedFullGGffgggyyojjhhjkhjyyiuhggjhhjhuyytytyuuytrrtghjtyuggyuighjuyioyyfgffhyuhhghzzrrkkhhkkkhhhjhkjhhghhgghiuiuhhjhjiuiuyujjkjjkint}
and
\er{MaxMedFullGGffgggyyojjhhjkhjyyiuhggjhhjhuyytytyuuytrrtghjtyuggyuighjuyioyyfgffhyuhhghzzrrkkhhkkkhjjhjjint}
we deduce
\begin{equation}\label{MaxMedFullGGffgggyyojjhhjkhjyyiuhggjhhjhuyytytyuuytrrtghjtyuggyuighjuyioyyfgffhyuhhghzzrrkkhhkkkhjjhjjjjkint}
\delta(-S)=\frac{1}{\sqrt{\kappa_0\gamma_0}}\left|\partial
\mathcal{S}_0\right|+\frac{2}{c}\iint \vec w\cdot\vec n
\,d\mathcal{S}_0.
\end{equation}
In particular, if the surface $\mathcal{S}_0$ is a part of some
plain then we rewrite
\er{MaxMedFullGGffgggyyojjhhjkhjyyiuhggjhhjhuyytytyuuytrrtghjtyuggyuighjuyioyyfgffhyuhhghzzrrkkhhkkkhjjhjjjjkint}
as
\begin{equation}\label{MaxMedFullGGffgggyyojjhhjkhjyyiuhggjhhjhuyytytyuuytrrtghjtyuggyuighjuyioyyfgffhyuhhghzzrrkkhhkkkhjjhjjjjkjjkjljklint}
\delta(-S)=\frac{1}{\sqrt{\kappa_0\gamma_0}}\left|\partial
\mathcal{S}_0\right|+\frac{2}{c} \left(\vec w\cdot\vec
n\right)\left|\mathcal{S}_0\right|.
\end{equation}
On the other hand, if the light travels across the same curve in the
opposite direction, then we must have:
\begin{equation}\label{MaxMedFullGGffgggyyojjhhjkhjyyiuhggjhhjhuyytytyuuytrrtghjtyuggyuighjuyioyyfgffhyuhhghzzrrkkhhkkkhjjhjjjjkjjkjljklpooint}
\delta(-S^-)=\frac{1}{\sqrt{\kappa_0\gamma_0}}\left|\partial
\mathcal{S}_0\right|-\frac{2}{c} \left(\vec w\cdot\vec
n\right)\left|\mathcal{S}_0\right|.
\end{equation}
Thus, by taking the difference in two cases and using
\er{MaxVacFullPPNmmmffffffiuiuhjuughbghhuiiujjhhjjhjhhjhjjhhjhjjhint},
we deduce:
\begin{equation}\label{MaxMedFullGGffgggyyojjhhjkhjyyiuhggjhhjhuyytytyuuytrrtghjtyuggyuighjuyioyyfgffhyuhhghzzrrkkhhkkkhjjhjjjjkjjkjljklkkkhjhjint}
\left(\delta(-T)-\delta(-T^-)\right)=k_0\left(\delta(-S)-\delta(-S^-)\right)=\frac{4\omega}{c^2}\,.
\left(\vec w\cdot\vec n\right)\left|\mathcal{S}_0\right|.
\end{equation}
Here, $\gamma_0$ and $\kappa_0$ are the dielectric and the magnetic
permeability of the medium, $T$ is given in
\er{MaxVacFullPPNmmmffffffiuiuhjuughbghhuiiujjhhjjhjhhjhjjhhjhjjhint},
$\left|\mathcal{S}_0\right|$ is the area of the flat surface bounded
by the closed path of the light, $\vec n$ is the unit normal to the
surface, $\omega$ is the frequency of the light and $\vec w$ is the
angular speed vector of the rotation of our coordinate system.

\subsubsection{Fizeau experiment}\label{seGOfzint}
Assume again the monochromatic electromagnetic wave of the frequency
$\nu=\omega/(2\pi)$ characterized by:
\begin{equation}\label{MaxVacFullPPNmmmffffffiuiuhjuughbghhuiiujjhhjjhjhhjhjjhhjhjjhfzint}
U(\vec x,t)=A(\vec x,t)e^{iT(\vec x,t)}=A(\vec x,t)e^{ik_0S(\vec
x,t)},\quad\text{where}\quad
k_0=\frac{\omega}{c}\quad\text{and}\quad\frac{\partial S}{\partial
t}=c\,.
\end{equation}
Then by \er{gughhghfbvnbvyyuuyrint} we have
\begin{equation}\label{gughhghfbvnbvyyuuyrhjhjhjfzint}
c_0=c\sqrt{\kappa_0\gamma_0}\quad\text{and}\quad\vec {\tilde
u}=\left(\gamma_0\kappa_0\vec v+(1-\gamma_0\kappa_0)\vec u\right),
\end{equation}
where, $\vec u$ is the medium velocity and $\vec v$ is the local
vectorial gravitational potential. Moreover, assume that we consider
light traveling in some region filled with the moving medium of
constant dielectric permeability $\gamma_0$ and magnetic
permeability $\kappa_0$. Then by \er{gughhghfbvnbvyyuuyrhjhjhjfzint}
and
\er{MaxMedFullGGffgggyyojjhhjkhjyyiuhggjhhjhuyytytyuuytrrtghjtyuggyuighjuyioyyfgffhyuhhghzzrrkkhhkkkhhhjhkjhhghhgghiuiu1int}
we have
\begin{equation}\label{gughhghfbvnbvyyuuyrhiyyuiuuhjhfzint}
n=\frac{c}{c_0}=\frac{1}{\sqrt{\kappa_0\gamma_0}}\;\;\;\text{is a
constant,}\quad\text{and}\quad\vec {\tilde u}=\frac{1}{n^2}\,\vec
v+\left(1-\frac{1}{n^2}\right)\vec u.
\end{equation}
Next, assume that the light travels from point $N$ to point $M$
across the curve $\vec r(s):[a,b]\to\mathbb{R}^3$ with endpoints
$\vec r(a)=N$ and $\vec r(b)=M$ undergoing possibly certain number
of reflections from mirrors during its travel. Then, as before, by
\er{MaxMedFullGGffgggyyojjhhjkhjyyiuhggjhhjhuyytytyuuytrrtghjtyuggyuighjuyioyyfgffhyuhhghzzrrkkhhkkkint},
\er{gughhghfbvnbvyyuuyrhiyyuiuuhjhfzint} and
\er{MaxMedFullGGffgggyyojjyugggjhhjiiuuiyuyuyyuyuzzhhkkkint} we
have:
\begin{multline}\label{MaxMedFullGGffgggyyojjhhjkhjyyiuhggjhhjhuyytytyuuytrrtghjtyuggyuighjuyioyyfgffhyuhhghzzrrkkhhkkkhjjhfzint}
\delta(-S):=\left(-S(M^-)\right)- \left(-S(N^+)\right)=\\
n\int_a^b \left|\vec r'(s)\right|ds-\frac{1}{c}\int_a^b\vec
v\left(\vec r(s)\right)\cdot\vec
r'(s)ds-\frac{n^2}{c}\left(1-\frac{1}{n^2}\right)\int_a^b\vec
u\left(\vec r(s)\right)\cdot\vec r'(s)ds.
\end{multline}
Next assume that, either our curve is perpendicular to the direction
of the vectorial gravitational potential $\vec v$, that happens, for
example, if our path of the light is tangent to the Earth surface,
or assume that our curve is closed, i.e. $M=N$ and moreover, our
coordinate system is inertial, or more generally non-rotating. In
particular, if we assume that $M=N$ i.e. our curve is closed and
moreover, our coordinate system is inertial, or more generally
non-rotating, then, as before, by Stokes Theorem we have:
\begin{equation}\label{MaxMedFullGGffgggyyojjhhjkhjyyiuhggjhhjhuyytytyuuytrrtghjtyuggyuighjuyioyyfgffhyuhhghzzrrkkhhkkkhjjhjjfzint}
\int_a^b\vec v\left(\vec r(s)\right)\cdot\vec r'(s)ds=0.
\end{equation}
On the other hand in the case that our curve is perpendicular to the
direction of the vectorial gravitational potential $\vec v$,
\er{MaxMedFullGGffgggyyojjhhjkhjyyiuhggjhhjhuyytytyuuytrrtghjtyuggyuighjuyioyyfgffhyuhhghzzrrkkhhkkkhjjhjjfzint}
also trivially follows. Therefore, by inserting
\er{MaxMedFullGGffgggyyojjhhjkhjyyiuhggjhhjhuyytytyuuytrrtghjtyuggyuighjuyioyyfgffhyuhhghzzrrkkhhkkkhjjhjjfzint}
into
\er{MaxMedFullGGffgggyyojjhhjkhjyyiuhggjhhjhuyytytyuuytrrtghjtyuggyuighjuyioyyfgffhyuhhghzzrrkkhhkkkhjjhfzint}
in both cases we obtain:
\begin{multline}\label{MaxMedFullGGffgggyyojjhhjkhjyyiuhggjhhjhuyytytyuuytrrtghjtyuggyuighjuyioyyfgffhyuhhghzzrrkkhhkkkhjjhfzfzfzint}
\delta(-S)=\left(-S(M^-)\right)- \left(-S(N^+)\right)= n\int_a^b
\left|\vec
r'(s)\right|ds-\frac{n^2}{c}\left(1-\frac{1}{n^2}\right)\int_a^b\vec
u\left(\vec r(s)\right)\cdot\vec r'(s)ds.
\end{multline}
Then by
\er{MaxMedFullGGffgggyyojjhhjkhjyyiuhggjhhjhuyytytyuuytrrtghjtyuggyuighjuyioyyfgffhyuhhghzzrrkkhhkkkhjjhfzfzfzint}
and
\er{MaxVacFullPPNmmmffffffiuiuhjuughbghhuiiujjhhjjhjhhjhjjhhjhjjhfzint}
we deduce
\begin{multline}\label{MaxMedFullGGffgggyyojjhhjkhjyyiuhggjhhjhuyytytyuuytrrtghjtyuggyuighjuyioyyfgffhyuhhghzzrrkkhhkkkhjjhfzfzfzfzffzint}
\delta(-T):=\left(-T(M^-)\right)-
\left(-T(N^+)\right)=k_0\delta(-S)\\= \frac{n\omega}{c}\int_a^b
\left|\vec
r'(s)\right|ds-\frac{n^2\omega}{c^2}\left(1-\frac{1}{n^2}\right)\int_a^b\vec
u\left(\vec r(s)\right)\cdot\vec r'(s)ds\\
\frac{n^2\omega}{c^2}\left(c_0\int_a^b \left|\vec
r'(s)\right|ds-\left(1-\frac{1}{n^2}\right)\int_a^b\vec u\left(\vec
r(s)\right)\cdot\vec r'(s)ds\right).
\end{multline}
In particular, if the absolute value $\left|\vec u\left(\vec
r(s)\right)\right|$ is a constant across the curve and if the angle
between $r'(s)$ and $\vec u\left(\vec r(s)\right)$ is a constant
across the curve and equals to the value $\theta$ then denoting the
length of the path by $L$:
\begin{equation}\label{MaxMedFullGGffgggyyojjhhjkhjyyiuhggjhhjhuyytytyuuytrrtghjtyuggyuighjuyioyyfgffhyuhhghzzrrkkhhkkkhjjhjjfzuyyint}
L:=\int_a^b \left|\vec r'(s)\right|ds,
\end{equation}
by
\er{MaxMedFullGGffgggyyojjhhjkhjyyiuhggjhhjhuyytytyuuytrrtghjtyuggyuighjuyioyyfgffhyuhhghzzrrkkhhkkkhjjhfzfzfzfzffzint}
we deduce:
\begin{equation}\label{MaxMedFullGGffgggyyojjhhjkhjyyiuhggjhhjhuyytytyuuytrrtghjtyuggyuighjuyioyyfgffhyuhhghzzrrkkhhkkkhjjhfzfzfzfzffzghghgfzint}
\delta(-T)=k_0\delta(-S)=\frac{\omega L
n^2}{c^2}\left(c_0-\left(1-\frac{1}{n^2}\right)|\vec
u|\cos{(\theta)}\right).
\end{equation}
Thus, if the direction of $\vec u$ coincides with the direction of
the light i.e. $\theta=0$ then
\begin{equation}\label{MaxMedFullGGffgggyyojjhhjkhjyyiuhggjhhjhuyytytyuuytrrtghjtyuggyuighjuyioyyfgffhyuhhghzzrrkkhhkkkhjjhfzfzfzfzffzghghgfzytyfzint}
\delta(-T)=k_0\delta(-S)=\frac{\omega L
n^2}{c^2}\left(c_0-\left(1-\frac{1}{n^2}\right)|\vec
u|\right)\approx\frac{\omega
L}{\left(c_0+\left(1-\frac{1}{n^2}\right)|\vec u|\right)}.
\end{equation}
On the other hand, if the direction of $\vec u$ is opposite to the
direction of the light i.e. $\theta=\pi$ then
\begin{equation}\label{MaxMedFullGGffgggyyojjhhjkhjyyiuhggjhhjhuyytytyuuytrrtghjtyuggyuighjuyioyyfgffhyuhhghzzrrkkhhkkkhjjhfzfzfzfzffzghghgfzuyuyhffzint}
\delta(-T)=k_0\delta(-S)=\frac{\omega L
n^2}{c^2}\left(c_0+\left(1-\frac{1}{n^2}\right)|\vec
u|\right)\approx\frac{\omega
L}{\left(c_0-\left(1-\frac{1}{n^2}\right)|\vec u|\right)}.
\end{equation}
So, in the frames of our model we explain the results of the Fizeau
experiment.


\subsubsection{The case of the non-monochromatic wave or/and  moving domains of propagation of light}\label{mnGObnbnint}
Next, assume that the wave is not monochromatic or/and the fields
$\vec {\tilde u}$ and $c_0$ depend on the time variable or/and we
consider the case of moving domains of propagation of light (in
particular moving surfaces of reflection/refraction). In other words
we can not assume
\er{MaxVacFullPPNmmmffffffiuiuhjuughbghhiuijghghghhjhjhhghyuyhhjhjihikjjjjkjljklint}
or
\er{MaxVacFullPPNmmmffffffiuiuhjuughbghhiuijghghghhjhjhhghyuyhhjhjihikjjjint}
anymore. However we do the Geometric Optics approximation
\er{slochangGGaaffgfgjhjgghint}. We also assume that our medium has
no dispersion. Then, due to
\er{MaxVacFullPPNmmmffffffiuiuhjuughbghhuiiujjint} we have:
\begin{equation}\label{MaxVacFullPPNmmmffffffiuiuhjuughbghhuiiujjhkhkhint}
U(\vec x,t)=A(\vec x,t)e^{ik_0S(\vec x,t)},
\end{equation}
and by
\er{MaxVacFullPPNmmmffffffiuiuhjuughbghhuiiujjjjjjjjhhhjjjint} and
\er{MaxVacFullPPNmmmffffffiuiuhjuughbghhuiiujjjjjjjjhhhjjjkkint} we
have:
\begin{equation}\label{MaxVacFullPPNmmmffffffiuiuhjuughbghhuiiujjjjjjjjhhhjjjjkkjjkihhhjint}
\left<\,\left|\frac{\partial S}{\partial t}\right|\,\right>=c,
\end{equation}
where the sign $\left<\cdot\right>$ means the spatial and temporal
averaging. Then, due to
\er{MaxVacFullPPNmmmffffffiuiuhjuughbghhiuijghghghhjhjhhghyuyjjjhhjhjffint}
we have the Eikonal type equation:
\begin{equation}\label{MaxVacFullPPNmmmffffffiuiuhjuughbghhiuijghghghhjhjhhghyuyjjjhhjhjffkklint}
\frac{1}{c^2_0}\left(\frac{\partial S}{\partial t}+\vec {\tilde
u}\cdot\nabla_\vec x S\right)^2=\left|\nabla_\vec x S\right|^2,
\end{equation}
and we rewrite the equation of the propagation of the beam
\er{MaxVacFullPPNmmmffffffiuiuhjuughbghhiuijghghghhhfhhghghguygtjjjkkjjkkhhint}
as:
\begin{equation}\label{MaxVacFullPPNmmmffffffiuiuhjuughbghhiuijghghghhhfhhghghguygtjjjkkjjkkhhklklkjkint}
\frac{\partial A}{\partial t}(\vec x,t)
+\vec h(\vec x,t)\cdot\nabla_{\vec x}A(\vec x,t)=
G(\vec x,t) A(\vec x,t)
,
\end{equation}
where we denote
\begin{equation}\label{MaxVacFullPPNmmmffffffiuiuhjuughbghhiuijghghghhhgjgfghjdfjgddg11int}
\vec h(\vec x,t):=\vec {\tilde u}(\vec x,t)-c^2_0(\vec
x,t)\left(\frac{\partial S}{\partial t}(\vec x,t)+\vec {\tilde
u}(\vec x,t)\cdot\nabla_{\vec x}S(\vec x,t)\right)^{-1}\nabla_{\vec
x}S(\vec x,t),
\end{equation}
and
\begin{multline}\label{MaxVacFullPPNmmmffffffiuiuhjuughbghhiuijghghghhhfhhghghguygtjuuujjjkhjhjh11kklkloiouuint}
G\left(\vec x,t\right):=\\
\left(\frac{c^2_0}{2\left(\frac{\partial S}{\partial t}+\vec {\tilde
u}\cdot\nabla_{\vec x}S\right)}\left(\Delta_{\vec
x}S-\frac{\partial}{\partial
t}\left\{\frac{1}{c^2_0}\left(\frac{\partial S}{\partial t}+\vec
{\tilde u}\cdot\nabla_{\vec x}S\right)\right\}-\Div_{\vec x}
\left\{\frac{1}{c^2_0}\left(\frac{\partial S}{\partial t}+\vec
{\tilde u}\cdot\nabla_{\vec x} S\right)\vec {\tilde
u}\right\}\right)\right).
\end{multline}
Next, consider a curve $\vec r(t):[t_0,b]\to\mathbb{R}^3$
parameterized by the time variable $t$, defined as a solution of
ordinary differential equation:
\begin{equation}\label{MaxVacFullPPNmmmffffffiuiuhjuughbghhiuijghghghhhfhhghghguygtjuuujjjkk11int}
\begin{cases}
\frac{d\vec r}{dt}(t)=\vec h\left(\vec r(t),t\right)\\
\vec r(t_0)=\vec x_0,
\end{cases}
\end{equation}
where $\vec h$ was defined in
\er{MaxVacFullPPNmmmffffffiuiuhjuughbghhiuijghghghhhgjgfghjdfjgddg11int}.
Note that the equations of the ray propagation of the form
\er{MaxVacFullPPNmmmffffffiuiuhjuughbghhiuijghghghhhfhhghghguygtjuuujjjkk11int}
are not always convenient since $\vec h$ in the right hand side of
\er{MaxVacFullPPNmmmffffffiuiuhjuughbghhiuijghghghhhfhhghghguygtjuuujjjkk11int}
depend on the partial derivatives of the quantity $S$, which we do
not know apriory unless we solve
\er{MaxVacFullPPNmmmffffffiuiuhjuughbghhiuijghghghhjhjhhghyuyjjjhhjhjffkklint}.
In the following proposition, that we prove in subsection
\ref{mnGObnbn}, we present the alternative form for the equations of
the ray propagation being the second order ordinary differential
equations which dose not contain the quantity $S$ or its partial
derivatives.
\begin{proposition}\label{profghgjjhint}
Consider a smooth solution of
\er{MaxVacFullPPNmmmffffffiuiuhjuughbghhiuijghghghhjhjhhghyuyjjjhhjhjffkklint}.
Next let $\vec r(t):[t_0,b]\to\mathbb{R}^3$ be any curve
parameterized by the time variable $t$, defined as a solution
\er{MaxVacFullPPNmmmffffffiuiuhjuughbghhiuijghghghhhfhhghghguygtjuuujjjkk11int}.
Then, $\vec r(t)$ satisfies the following second order ordinary
differential equations of the Ray Propagation:
\begin{multline}\label{MaxVacFullPPNmmmffffffiuiuhjuughbghhiuijghghghhhfhhghghguygtjuuujjjkk11jjjggjjgtttjhhuhhjhjhgghghjhjggjugghghint}
\frac{d}{dt}\left(\frac{1}{c_0\left(\vec
r(t),t\right)}\left(\frac{d\vec r}{dt}(t)-\vec {\tilde u}\left(\vec
r(t),t\right)\right)\right)=
 -\nabla_\vec x c_0(\vec r,t)-\frac{1}{c_0(\vec r,t)}\left\{d_{\vec x}\vec {\tilde u}(\vec
r,t)\right\}^T\cdot\left(\frac{d\vec r}{dt}-\vec {\tilde
u}(\vec r,t)\right)\\
+\frac{1}{c^2_0(\vec r,t)}\left(\left(\frac{d\vec r}{dt}-\vec
{\tilde u}(\vec r,t)\right)\cdot\nabla_{\vec x}c_0(\vec
r,t)\right)\left(\frac{d\vec r}{dt}-\vec {\tilde
u}(\vec r,t)\right)\\
+\frac{1}{2c^3_0(\vec r,t)}\left(\left(\left(\left\{d_{\vec x}\vec
{\tilde u}(\vec r,t)\right\}^T+d_{\vec x}\vec {\tilde u}(\vec
r,t)\right)\cdot\left(\frac{d\vec r}{dt}-\vec {\tilde u}(\vec
r,t)\right)\right)\cdot \left(\frac{d\vec r}{dt}-\vec {\tilde
u}(\vec r,t)\right)\right)\left(\frac{d\vec r}{dt}-\vec {\tilde
u}(\vec r,t)\right).
\end{multline}
Moreover, the proper scalar quantity
$\tilde\omega(t):[t_0,b]\to\mathbb{R}$, defined by
\begin{equation}\label{MaxVacFullPPNmmmffffffiuiuhjuughbghhiuijghghghhhfhhghghguygtjuuujjjkk11jjjggjjgyhhjhhjint}
\tilde\omega(t)=\frac{\partial S}{\partial t}\left(\vec
r(t),t\right)+\vec {\tilde u}\left(\vec
r(t),t\right)\cdot\nabla_{\vec x}S\left(\vec r(t),t\right)
\end{equation}
satisfies the following ordinary differential equation:
\begin{multline}\label{MaxVacFullPPNmmmffffffiuiuhjuughbghhiuijghghghhhfhhghghguygtjuuujjjkk11jjjggjjgkhhhhkjhkhkhhkkkhkkhhhhjggghjhjint}
\frac{d\tilde\omega}{dt}(t)= \frac{1}{c_0(\vec
r,t)}\left(\frac{\partial c_0}{\partial t}(\vec r,t)+\vec {\tilde
u}(\vec r,t)\cdot\nabla_{\vec x}c_0(\vec
r,t)\right)\tilde\omega(t)\\-\frac{1}{2c^2_0(\vec
r,t)}\left(\left(\left(\left\{d_{\vec x}\vec {\tilde u}(\vec
r,t)\right\}^T+d_{\vec x}\vec {\tilde u}(\vec
r,t)\right)\cdot\left(\frac{d\vec r}{dt}(t)-\vec {\tilde u}(\vec
r,t)\right)\right)\cdot \left(\frac{d\vec r}{dt}(t)-\vec {\tilde
u}(\vec r,t)\right)\right)\tilde\omega(t).
\end{multline}
\end{proposition}
Note that, as we deduce in subsection \ref{mnGObnbn}, if $\vec r(t)$
satisfies
\er{MaxVacFullPPNmmmffffffiuiuhjuughbghhiuijghghghhhfhhghghguygtjuuujjjkk11jjjggjjgtttjhhuhhjhjhgghghjhjggjugghghint}
and moreover it satisfies the following equality for the initial
instant of time $t_0$:
\begin{equation*}
\frac{1}{c^2_0\left(\vec r(t_0),t_0\right)}\left|\frac{d\vec
r}{dt}(t_0)-\vec {\tilde u}\left(\vec r(t_0),t_0\right)\right|^2=1,
\end{equation*}
then we have
\begin{equation*}
\frac{1}{c^2_0\left(\vec r(t),t\right)}\left|\frac{d\vec
r}{dt}(t)-\vec {\tilde u}\left(\vec r(t),t\right)\right|^2=1
\end{equation*}
for every instant of time. So we have the following consistency with
\er{MaxVacFullPPNmmmffffffiuiuhjuughbghhiuijghghghhjhjhhghyuyjjjhhjhjffkklkljljjkhhjkjhghghghl;kl;;k;kkljjkjkll;l;int}
and
\er{MaxVacFullPPNmmmffffffiuiuhjuughbghhiuijghghghhhfhhghghguygtjuuujjjkk11int}:
\begin{equation}\label{MaxVacFullPPNmmmffffffiuiuhjuughbghhiuijghghghhhfhhghghguygtjuuujjjkk11jjjggjjgjkljkuhuhikjhjhint}
\left|\frac{d\vec r}{dt}(t)-\vec {\tilde u}(\vec
r,t)\right|^2=c^2_0(\vec r,t).
\end{equation}

Next note that, as before, we can easily deduce that equations of
the ray propagation either of the form
\er{MaxVacFullPPNmmmffffffiuiuhjuughbghhiuijghghghhhfhhghghguygtjuuujjjkk11int}
or of the form
\er{MaxVacFullPPNmmmffffffiuiuhjuughbghhiuijghghghhhfhhghghguygtjuuujjjkk11jjjggjjgtttjhhuhhjhjhgghghjhjggjugghghint}
are invariant under the change of non-inertial cartesian coordinate
system. Moreover,
\er{MaxVacFullPPNmmmffffffiuiuhjuughbghhiuijghghghhhfhhghghguygtjuuujjjkk11jjjggjjgkhhhhkjhkhkhhkkkhkkhhhhjggghjhjint}
is also invariant under the change of non-inertial cartesian
coordinate system.

Next consider a wave characterized by:
\begin{equation}\label{MaxVacFullPPNmmmffffffiuiuhjuughbghhuiiujjhhjjhjhhj11int}
U(\vec x,t)=A(\vec x,t)e^{ik_0S(\vec x,t)},
\end{equation}
and, consistently with
\er{MaxVacFullPPNmmmffffffiuiuhjuughbghhiuijghghghhhgjgfghjdfjgddg11int}
consider a velocity field of the wave:
\begin{equation}\label{MaxVacFullPPNmmmffffffiuiuhjuughbghhiuijghghghhhfhhghghguygtjuuujjjkyuuykjkjj11int}
\vec h(\vec x,t):=\vec {\tilde u}(\vec x,t)-c^2_0(\vec
x,t)\left(\frac{\partial S}{\partial t}(\vec x,t)+\vec {\tilde
u}(\vec x,t)\cdot\nabla_{\vec x}S(\vec x,t)\right)^{-1}\nabla_{\vec
x}S(\vec x,t),
\end{equation}
Furthermore, assume that the wave we consider undergoes reflection
and/or refraction on the time-dependent surface $\mathcal{T}$
having the outcoming three-dimensional unit normal $\vec n(\vec
x,t)$ and the motion velocity field $\vec w_\mathcal{T}(\vec x,t)$,
separating two regions characterized respectively by $c_0=c^{(1)}_0$
and $\vec {\tilde u}=\vec {\tilde u}_1$ and by $c^{(2)}_0$ and $\vec
{\tilde u}_2$, with the formation of the reflected wave,
characterized by:
\begin{equation}\label{MaxVacFullPPNmmmffffffiuiuhjuughbghhuiiujjhhjjhjhhjugh11int}
U_1(\vec x,t)=A_1(\vec x,t)e^{ik_0S_1(\vec x,t)}\,,
\end{equation}
and by the velocity field:
\begin{equation}\label{MaxVacFullPPNmmmffffffiuiuhjuughbghhiuijghghghhhfhhghghguygtjuuujjjkyuuykjkjjhjhhj11int}
\vec h_1(\vec x)=\vec {\tilde u}_1(\vec x,t)-c^2_0(\vec
x,t)\left(\frac{\partial S_1}{\partial t}(\vec x,t)+\vec {\tilde
u}_1(\vec x,t)\cdot\nabla_{\vec x}S_1(\vec
x,t)\right)^{-1}\nabla_{\vec x}S_1(\vec x,t),
\end{equation}
and formation of the refracted wave characterized by:
\begin{equation}\label{MaxVacFullPPNmmmffffffiuiuhjuughbghhuiiujjhhjjhjhhj111int}
U_2(\vec x,t)=A_2(\vec x,t)e^{ik_0S_2(\vec x,t)}\,,
\end{equation}
and by the velocity field:
\begin{equation}\label{MaxVacFullPPNmmmffffffiuiuhjuughbghhiuijghghghhhfhhghghguygtjuuujjjkyuuykjkjj111int}
\vec h_2(\vec x)=\vec {\tilde u}_2(\vec x,t)-(c^{(2)}_0)^2(\vec
x,t)\left(\frac{\partial S_2}{\partial t}(\vec x,t)+\vec {\tilde
u}_2(\vec x,t)\cdot\nabla_{\vec x}S_2(\vec
x,t)\right)^{-1}\nabla_{\vec x}S_2(\vec x,t).
\end{equation}
Then the boundary conditions of $U$, $U_1$ and $U_2$ depend on the
physical meaning of these fields. However, one of the
\underline{necessary} conditions should be that
\begin{equation}\label{MaxMedFullGGffgggyyojjyugggjhhjiiuuiyuyuyyuyuzzhhkkk11int}
S_1(\vec x,t)=S_2(\vec x,t)+C_2=S(\vec x,t)\quad\quad\forall\vec
x\in\mathcal{T},\;\forall t,
\end{equation}
where $C_2$ is a real constant. In particular
\er{MaxMedFullGGffgggyyojjyugggjhhjiiuuiyuyuyyuyuzzhhkkk11int}
implies
\begin{equation}\label{MaxMedFullGGffgggyyojjyugggjhhjiiuuiyuyuyyuyuzzhhkkkgdfg11int}
\nabla_{\vec x} S_1-\left(\vec n\cdot \nabla_{\vec x} S_1\right)\vec
n
=\nabla_{\vec x} S_2-\left(\vec n\cdot \nabla_{\vec x}
S_2\right)\vec n=\nabla_{\vec x} S-\left(\vec n\cdot \nabla_{\vec x}
S\right)\vec n\quad\quad\forall\vec x\in\mathcal{T},\;\forall t.
\end{equation}
In particular, for every point on the surface $\mathcal{T}$ vectors
$\nabla_{\vec x} S_1$ and $\nabla_{\vec x} S_2$ lie in the plane
formed by vectors $\vec n$ and $\nabla_{\vec x} S$. Moreover, by
\er{MaxMedFullGGffgggyyojjyugggjhhjiiuuiyuyuyyuyuzzhhkkk11int} we
also have
\begin{equation}\label{MaxMedFullGGffgggyyojjyugggjhhjiiuuiyuyuyyuyuzzhhkkkgdfghgghgh11kjhk11int}
\frac{\partial S}{\partial t}+\vec w_\mathcal{T}\cdot\nabla_{\vec x}
S=\frac{\partial S_1}{\partial t}+\vec
w_\mathcal{T}\cdot\nabla_{\vec x} S_1=\frac{\partial S_2}{\partial
t}+\vec w_\mathcal{T}\cdot\nabla_{\vec x} S_2\quad\quad\forall\vec
x\in\mathcal{T},\;\forall t.
\end{equation}
Finally, by
\er{MaxVacFullPPNmmmffffffiuiuhjuughbghhiuijghghghhjhjhhghyuyjjjhhjhjffkklkljljjkhhjkjhghghghl;kl;;k;kkljjkjkll;l;int}
we have:
\begin{equation}\label{MaxVacFullPPNmmmffffffiuiuhjuughbghhiuijghghghhjhjhhghyuyjjjhhjhjffkklkljljjkhhjkjhghghghl;kl;;k;kkljjkjkll;l;jijhhjkjkkjint}
\left|\vec h-\vec{\tilde u}\right|^2=c^2_0\,,\quad\left|\vec
h_1-\vec{\tilde u}\right|^2=c^2_0\quad\text{and}\quad\left|\vec
h_2-\vec{\tilde u}_2\right|^2=(c^{(2)}_0)^2.
\end{equation}
In subsection \ref{mnGObnbn} we prove that if the following
approximation is valid on the both sides of the surface
$\mathcal{T}$:
\begin{equation}\label{ojhkkkpkkljouu11int}
\frac{|\vec w_\mathcal{T}|^2}{c^2_0+(c^{(2)}_0)^2}\ll
1\,,\quad\frac{|\vec {\tilde u}|^2}{c^2_0}\ll
1\quad\text{and}\quad\frac{|\vec {\tilde u}_2|^2}{(c^{(2)}_0)^2}\ll
1,
\end{equation}
then, we have the following law of reflection: vector $(\vec
h_1-\vec w_{\mathcal{T}})$ lies in the plane formed by vectors $\vec
n$ and $(\vec h-\vec w_{\mathcal{T}})$, and we have:
\begin{equation}\label{MaxMedFulljhhjjjjj11int}
\theta\left((\vec h-\vec w_{\mathcal{T}}),-\vec
n\right)=\theta_1\left((\vec h_1-\vec w_{\mathcal{T}}),\vec n\right)
\end{equation}
where $\theta\left((\vec h-\vec w_{\mathcal{T}}),-\vec n\right)$ is
the angle between the vector of the relative velocity of the
incoming ray, relative to the surface of reflection, $(\vec h-\vec
w_{\mathcal{T}})$ and the incoming normal to the surface $-\vec n$
and $\theta_1\left((\vec h_1-\vec w_{\mathcal{T}}),\vec n\right)$ is
the angle between the vector of the relative velocity of the
reflected ray, relative to the surface of reflection, $(\vec
h_1-\vec w_{\mathcal{T}})$ and the outcoming normal $\vec n$.

Moreover, if we assume that the wave we consider in
\er{MaxVacFullPPNmmmffffffiuiuhjuughbghhuiiujjhhjjhjhhj11int} has an
electromagnetic nature. Then by \er{gughhghfbvnbvyyuuyrint} we have
\begin{equation}\label{gughhghfbvnbvyyuuy11int}
c_0=c\sqrt{\kappa_0\gamma_0}\quad\text{and}\quad\vec {\tilde
u}=\left(\gamma_0\kappa_0\vec v+(1-\gamma_0\kappa_0)\vec u\right),
\end{equation}
where, $\vec u$ is the medium velocity and $\vec v$ is the local
vectorial gravitational potential. Similarly, on the second side of
surface $\mathcal{T}$ we have
\begin{equation}\label{gughhghfbvnbvyyuuyGGHHG11int}
c^{(2)}_0=c\sqrt{\kappa^{(2)}_0\gamma^{(2)}_0}\quad\text{and}\quad\vec
{\tilde u}_2=\left(\gamma^{(2)}_0\kappa^{(2)}_0\vec
v+(1-\gamma^{(2)}_0\kappa^{(2)}_0)\vec u_2\right),
\end{equation}
where, $\vec u_2$ is the medium velocity on the second side of
surface $\mathcal{T}$. Then in subsection \ref{mnGObnbn} we also
prove that we have the following Snell's law of refraction: $(\vec
h_2-\vec w_{\mathcal{T}})$ lies in the plane formed by vectors $\vec
n$ and $(\vec h-\vec w_{\mathcal{T}})$ and we have:
\begin{equation}\label{MaxMedFulljhhjjjjjhhhj11int}
n\sin{\left(\theta\left((\vec h-\vec w_{\mathcal{T}}),\vec
n\right)\right)}=n_2\sin{\left(\theta_2\left((\vec h_2-\vec
w_{\mathcal{T}}),\vec n\right)\right)}
\end{equation}
where $\theta\left((\vec h-\vec w_{\mathcal{T}}),\vec n\right)$ is
the angle between the vector of the relative velocity of the
incoming ray, relative to the surface of refraction, $(\vec h-\vec
w_{\mathcal{T}})$ and the normal to the surface $\vec n$,
$\theta_2\left((\vec h_2-\vec w_{\mathcal{T}}),\vec n\right)$ is the
angle between the vector of the relative velocity of the refracted
ray, relative to the surface of refraction, $(\vec h_2-\vec
w_{\mathcal{T}})$ and the normal $\vec n$ and as in
\er{MaxMedFullGGffgggyyojjhhjkhjyyiuhggjhhjhuyytytyuuytrrtghjtyuggyuighjuyioyyfgffhyuhhghzzrrkkhhkkkhhhjhkjhhghhgghiuiu1int}
we set refraction indexes:
\begin{equation}\label{MaxMedFullGGffgggyyojjhhjkhjyyiuhggjhhjhuyytytyuuytrrtghjtyuggyuighjuyioyyfgffhyuhhghzzrrkkhhkkkhhhjhkjhhghhgghiuiujjkjk11int}
n:=\frac{c}{c_0}\quad\text{and}\quad n_2:=\frac{c}{c^{(2)}_0}.
\end{equation}
Note that, since $\vec h$ and $\vec w_{\mathcal{T}}$ are both speed
like vector fields then $(\vec h-\vec w_{\mathcal{T}})$ is a
\underline{proper} vector field and thus the above law of reflection
together with \er{MaxMedFulljhhjjjjj11int} and the Snell's law
together with \er{MaxMedFulljhhjjjjjhhhj11int} are invariant under
the change of inertial or non-inertial cartesian coordinate systems.
In particular, if \er{ojhkkkpkkljouu11int} holds for some cartesian
coordinate system, then we can use this laws also in other
coordinate systems where \er{ojhkkkpkkljouu11int} does not hold.
Therefore, for the validity of the above laws of reflection and
refraction we may assume the following relation instead of
\er{ojhkkkpkkljouu11int}:
\begin{equation}\label{ojhkkkpkkljouu11uihkhkjkjkkliyuttytyint}
\frac{|\vec {\tilde u}-\vec w_\mathcal{T}|^2}{c^2_0}\ll
1\quad\text{and}\quad\frac{|\vec {\tilde u}_2-\vec
w_\mathcal{T}|^2}{(c^{(2)}_0)^2}\ll 1.
\end{equation}

\subsubsection{Polarization of the light inside a moving medium and/or in the presence of gravitational
field}\label{jgjghjghhffhfint}
Again consider the system of Maxwell equations in the vacuum or in a
medium of the form \er{MaxMedFullGGffykkhjhhzzkk} in the absence of
macroscopic charges and/or currents:
\begin{equation}\label{MaxVacFullPPNffGGhjhjhjhint}
\begin{cases}
curl_{\vec x} \vec H=
\frac{1}{c}\frac{\partial \vec D}{\partial t},\\
div_{\vec x} \vec D=0
,\\
curl_{\vec x} \vec E+\frac{1}{c}\frac{\partial \vec B}{\partial t}=0,\\
div_{\vec x} \vec B=0,\\
\vec E=\gamma_0\vec D-\frac{1}{c}\,\vec {\tilde u}\times \vec B,\\
\vec H=\kappa_0\vec B+\frac{1}{c}\,\vec {\tilde u}\times \vec D,
\\ \vec {\tilde u}=\left(\gamma_0\kappa_0\vec v+(1-\gamma_0\kappa_0)\vec u\right),
\end{cases}
\end{equation}
and consider the case of monochromatic wave of the constant
frequency $\nu=\frac{\omega}{2\pi}$ where the fields $\vec {\tilde
u}$, $\gamma_0$ and $\kappa_0$ are independent of the time variable.
We also assume that either our medium has no dispersion, i.e.
\er{EBDHTrans444knbuihuigGGint} is valid, or the velocity of the
medium is negligible i.e. $\vec u\equiv 0$ (and so $\vec{\tilde
u}\equiv \vec v$) and our medium is transparent for the given
frequency $\nu=\frac{\omega}{2\pi}$.

Next assume the \underline{rough} Geometric Optics approximation
\er{slochangGGaaffgfgjhjggh11int} (stronger than
\er{slochangGGaaffgfgjhjgghint}) that means the following: assume
that the changes in time of the basic characteristics of the
electromagnetic field become essential after certain interval of
time $T_e$ and the changes in space of of the basic characteristics
of the electromagnetic field become essential in the spatial
landscape $L_e$. Then we assume that
\begin{equation}\label{slochangGGaaffgfgjhjgghhkjhklkllkint}
k_0c_0T_e\,\gg\, 1\quad\text{and}\quad k_0L_e\,\gg\, 1,
\end{equation}
where
\begin{equation}\label{MaxVacFullPPNmmmffffffiuiuhjuughbghhuiiujjhhjjhjhhjyuyhhjjjkjkint}
k_0=\frac{\omega}{c}.
\end{equation}
We also assume \er{ojhkkint}:
\begin{equation}\label{ojhkkkhhhhint}
\frac{|\vec {\tilde u}|^2}{c^2_0}\ll 1.
\end{equation}
Then consistently with
\er{MaxVacFullPPNmmmffffffiuiuhjuughbghhuiiujjint} we can write
\begin{equation}\label{MaxVacFullPPNmmmffffffiuiuhjuughbghhuiiujjnewjikhkjhjkkjhhjkkhhhjint}
\begin{cases}
\vec D(\vec x,t)=\Xi_1\cdot\vec D_a(\vec x)e^{ik_0S(\vec x,t)}
\\
\vec B(\vec x,t)=\Xi_2\cdot\vec B_a(\vec x)e^{ik_0S(\vec x,t)}
\\
\vec E(\vec x,t)=\Xi_3\cdot\vec E_a(\vec x)e^{ik_0S(\vec x,t)}
\\
\vec H(\vec x,t)=\Xi_4\cdot\vec H_a(\vec x)e^{ik_0S(\vec x,t)}
\\
\vec E=\gamma_0\vec D-\frac{1}{c}\,\vec {\tilde u}\times \vec B\\
\vec H=\kappa_0\vec B+\frac{1}{c}\,\vec {\tilde u}\times \vec D.
\end{cases}
\end{equation}
Here $\vec D_a(\vec x),\vec B_a(\vec x),\vec E_a(\vec x),\vec
H_a(\vec x):\mathbb{R}^3\to\mathbb{R}^3$ are real vector fields,
independent of the time variable, $S(\vec x,t)$ is a real function
such that $\frac{\partial S}{\partial t}=c$ and
$\Xi_1,\Xi_2,\Xi_3,\Xi_4\in \mathbb{C}^{3\times 3}$ are constant
complex diagonal matrices of the form:
\begin{equation}\label{ghgyfftdtdljkljnnewkkkkkint}
\Xi_k=\left(\begin{matrix}e^{i\theta_{1k}}&0&0\\0&e^{i\theta_{2k}}&0\\0&0&e^{i\theta_{3k}}\end{matrix}\right)\quad\quad\quad\quad\forall
k=1,2,3,4,
\end{equation}
where $\theta_{1k},\theta_{2k},\theta_{3k}$ are real constants.
%
%
%
%
%
%
%
Then, consistently with
\er{MaxVacFullPPNmmmffffffiuiuhjuughbghhiuijghghghhjhjhhghyuyiyyujjljkint},
$S$ satisfies the Eikonal equation:
\begin{equation}\label{MaxVacFullPPNmmmffffffiuiuhjuughbghhiuijghghghhhfhhghghguygtjuuujjjknewljjljllkgugghklkllljojjjint}
\left|\frac{c\vec {\tilde u}}{c^2_0}-\nabla_\vec x
S\right|^2=\frac{c^2}{c^2_0},
\end{equation}
and consistently with
\er{MaxVacFullPPNmmmffffffiuiuhjuughbghhiuijghghghhhfhhghghguygtjuuujjjkint}
if $\vec A_1$ denotes one of the vectors $\vec D_a,\vec B_a,\vec
E_a,\vec H_a$ then:
\begin{equation}\label{MaxVacFullPPNmmmffffffiuiuhjuughbghhiuijghghghhhfhhghghguygtjuuujjjknewljjljllkgugghklkllljjkjkkkklint}
\left\{d_{\vec x}\vec A_1\right\}^T\cdot\left(\frac{c\vec {\tilde
u}}{c^2_0}-\nabla_{\vec x}S\right)+\frac{1}{2}\left(-\Delta_{\vec
x}S\right)\vec A_1=0.
\end{equation}
Moreover, consistently with
\er{MaxVacFullPPNmmmffffffiuiuhjuughbghhiuijghghghhjhjhhghyuyhhjhjihikjjjjkjljklint},
\er{MaxVacFullPPNmmmffffffiuiuhjuughbghhiuijghghghhjhjhhghyuyhhjhjihikjjjint}
and \er{gughhghfbvnbvyyuuyrint} we have:
\begin{equation}\label{MaxVacFullPPNmmmffffffiuiuhjuughbghhiuijghghghhjhjhhghyuyhhjhjihikjjjnewint}
\begin{cases}
c_0=c\sqrt{\kappa_0\gamma_0}\\
\frac{\partial S}{\partial t}=c\\
\frac{\partial \vec A_1}{\partial t}=0,\\
\frac{\partial \vec {\tilde u}}{\partial t}=0\\
\frac{\partial c_0}{\partial t}=0,
\end{cases}
\end{equation}
and, consistently with
\er{MaxVacFullPPNmmmffffffiuiuhjuughbghhiuijghghghhhfhhghghguygtjuuujjjkyuuyint},
the vector field $\vec h_0$ defined for every $\vec x$ by:
\begin{equation}\label{MaxVacFullPPNmmmffffffiuiuhjuughbghhiuijghghghhhfhhghghguygtjuuujjjkyuuykkhhkint}
\vec h_0(\vec x):=\frac{c}{c^2_0(\vec x)}\vec {\tilde u}(\vec
x)-\nabla_{\vec x}S(\vec x),
\end{equation}
is the direction of the propagation of the ray that passes through
point $\vec x$.
Then in subsection \ref{jgjghjghhffhf} we deduce that:
\begin{equation}\label{MaxVacFullPPNffGGhjhjhjhjkhkhklklkint}
\begin{cases}
 \vec
B\approx\frac{\gamma_0}{\left(1+\frac{1}{c}\vec {\tilde
u}\cdot\nabla_{\vec x}S\right)}\left(-\nabla_{\vec
x}S\right)\times\vec D
\\
\vec D\approx-\frac{\kappa_0}{\left(1+\frac{1}{c}\vec {\tilde
u}\cdot\nabla_{\vec x}S\right)}\left(-\nabla_{\vec
x}S\right)\times\vec B
\\
\left(-\nabla_{\vec x}S\right)\cdot \vec D\approx 0\\
 \left(-\nabla_{\vec x}S\right)\cdot \vec
B\approx 0.
\end{cases}
\end{equation}
So the vectors $(-\nabla_{\vec x}S)\,$, $\,\vec D$ and $\vec B$ form
together rightly orientated orthogonal system of vectors.
Moreover, in subsection \ref{jgjghjghhffhf} we also deduce that:
\begin{equation}\label{MaxVacFullPPNffGGhjhjhjhjkhkhljjkjnhhhhnhkhkhhjhjhjjhjgjlljkjuiyklkint}
\begin{cases}
\vec H\approx\kappa_0
\vec h_0\times\vec E
\\
\vec E\approx
-\gamma_0\vec h_0\times\vec H\\
\vec h_0\cdot\vec E\approx 0\\
\vec h_0\cdot \vec H\approx 0.
\end{cases}
\end{equation}
So the vectors
$\vec h_0\,$, $\vec E$ and $\vec H$ form together rightly orientated
orthogonal system of vectors. We remind here again that $\vec h_0$
is the direction of the propagation of the ray.

\subsubsection{More delicate approximation of the wave
equation}\label{gjgggggggggggggggggggghkhkint} Again consider the
scalar wave equation of the form
\begin{equation}\label{MaxVacFullPPNmmmffffffiuiuhjuughbghhjjjoojouiuiuint}
\left(\frac{\partial}{\partial
t}\left\{\frac{1}{c^2_0}\left(\frac{\partial U}{\partial t}+\vec
{\tilde u}\cdot\nabla_{\vec x}U\right)\right\}+\Div_{\vec x}
\left\{\frac{1}{c^2_0}\left(\frac{\partial U}{\partial t}+\vec
{\tilde u}\cdot\nabla_{\vec x} U\right)\vec {\tilde
u}\right\}\right)-\Delta_{\vec x}U=0.
\end{equation}
Above, we established $U(\vec x,t)$ only in the \underline{rough}
Geometric Optics approximation \er{slochangGGaaffgfgjhjggh11int}.
Since we have
\er{MaxVacFullPPNmmmffffffiuiuhjuughbghhuiiujjjjjjjjgjjgint} in the
rough Geometric Optics approximation and $S(\vec x,t)$ satisfies the
following Eikonal equation
\begin{equation}\label{MaxVacFullPPNmmmffffffiuiuhjuughbghhiuijghghghhjhjhhghyuyjjjhhjhjffhjjhjojint}
\frac{1}{c^2_0}\left(\frac{\partial S}{\partial t}+\vec {\tilde
u}\cdot\nabla_\vec x S\right)^2=\left|\nabla_\vec x S\right|^2.
\end{equation} in the case of more precise
delicate Geometric Optics approximation, in order to get better
approximation for $U(\vec x,t)$ we can use the following alternative
consideration, slightly different than we did above: from now we
consider $U(\vec x,t)$ be of the form
\begin{equation}\label{MaxVacFullPPNmmmffffffiuiuhjuughbghhuiiujjjjjjjjgjjgiuiuuiint}
U(\vec x,t)=A(\vec x,t)e^{ik_0S(\vec x,t)}\,,
\end{equation}
similarly to
\er{MaxVacFullPPNmmmffffffiuiuhjuughbghhuiiujjjjjjjjgjjgint}.
However, we consider that $S$ in
\er{MaxVacFullPPNmmmffffffiuiuhjuughbghhuiiujjjjjjjjgjjgiuiuuiint}
satisfied
\er{MaxVacFullPPNmmmffffffiuiuhjuughbghhiuijghghghhjhjhhghyuyjjjhhjhjffhjjhjojint}
\underline{precisely}, although we consider the amplitude $A$ to be
\underline{complex} rather than real! So we include the correction
term of $\kappa_0 S$ into the complex amplitude $A$! Thus either
\er{MaxVacFullPPNmmmffffffiuiuhjuughbghhiuijghghghhhfhhghghguygtjjjkkjjkkhhint}
or
\er{MaxVacFullPPNmmmffffffiuiuhjuughbghhiuijghghghhhfhhghghguygtjjjkkjjkkhhklklkjk33hhjjjljint}
is not valid anymore, since it was derived under the assumption of
real amplitude $A$. However, either
\er{MaxVacFullPPNmmmffffffiuiuhjuughbghhiuijghghghhjhjhhghyuyjjjhhjhjffint}
or
\er{MaxVacFullPPNmmmffffffiuiuhjuughbghhiuijghghghhhgjgfghjdfjgddg1133khhujuuuo99klklljljljint}
still holds. I.e. we have
\er{MaxVacFullPPNmmmffffffiuiuhjuughbghhiuijghghghhjhjhhghyuyjjjhhjhjffhjjhjojint},
which assumed to hold precisely. Moreover, by
\er{MaxVacFullPPNmmmffffffiuiuhjuughbghhiuijghghghhjhjhjjuiuiuiint}
together with
\er{MaxVacFullPPNmmmffffffiuiuhjuughbghhiuijghghghhjhjhhghyuyjjjhhjhjffhjjhjojint}
we have
\begin{multline}\label{MaxVacFullPPNmmmffffffiuiuhjuughbghhiuijghghghhjhjhjjuiuiuiuuiuiint}
\left(\frac{\partial}{\partial
t}\left\{\frac{1}{c^2_0}\left(\frac{\partial A}{\partial t}+\vec
{\tilde u}\cdot\nabla_{\vec x}A\right)\right\}+\Div_{\vec x}
\left\{\frac{1}{c^2_0}\left(\frac{\partial A}{\partial t}+\vec
{\tilde u}\cdot\nabla_{\vec x} A\right)\vec {\tilde
u}\right\}\right)-\Delta_{\vec
x}A\\
+i\kappa_0 A\left(\frac{\partial}{\partial
t}\left\{\frac{1}{c^2_0}\left(\frac{\partial S}{\partial t}+\vec
{\tilde u}\cdot\nabla_{\vec x}S\right)\right\}+\Div_{\vec x}
\left\{\frac{1}{c^2_0}\left(\frac{\partial S}{\partial t}+\vec
{\tilde u}\cdot\nabla_{\vec x} S\right)\vec {\tilde
u}\right\}-\Delta_{\vec
x}S\right)\\+2i\kappa_0\left(\frac{1}{c^2_0}\left(\frac{\partial
S}{\partial t}+\vec {\tilde u}\cdot\nabla_{\vec
x}S\right)\left(\frac{\partial A}{\partial t}+\vec {\tilde
u}\cdot\nabla_{\vec x}A\right)-\nabla_{\vec x}A\cdot\nabla_{\vec
x}S\right)=0.
\end{multline}
Then denoting, the proper scalar field:
\begin{equation}\label{MaxVacFullPPNmmmffffffiuiuhjuughbghhiuijghghghhhgjgfghjdfjgddg1133khhujuuuo33int}
\tilde\omega(\vec x,t):=\kappa_0\left(\frac{\partial S}{\partial
t}\left(\vec x,t\right)+\vec {\tilde u}\left(\vec
x,t\right)\cdot\nabla_{\vec x}S\left(\vec x,t\right)\right),
\end{equation}
the proper vector field:
\begin{equation}\label{MaxVacFullPPNmmmffffffiuiuhjuughbghhiuijghghghhhgjgfghjdfjgddg11335533int}
\vec k(\vec x,t):=c_0(\vec x,t)\left(\frac{\partial S}{\partial
t}(\vec x,t)+\vec {\tilde u}(\vec x,t)\cdot\nabla_{\vec x}S(\vec
x,t)\right)^{-1}\nabla_{\vec x}S(\vec x,t),
\end{equation}
the speed-like vector field:
\begin{equation}\label{MaxVacFullPPNmmmffffffiuiuhjuughbghhiuijghghghhhgjgfghjdfjgddg113333int}
\vec h(\vec x,t):=\vec {\tilde u}(\vec x,t)-c^2_0(\vec
x,t)\left(\frac{\partial S}{\partial t}(\vec x,t)+\vec {\tilde
u}(\vec x,t)\cdot\nabla_{\vec x}S(\vec x,t)\right)^{-1}\nabla_{\vec
x}S(\vec x,t)=\vec {\tilde u}(\vec x,t)-c_0(\vec x,t)\vec k(\vec
x,t),
\end{equation}
and the proper scalar field:
\begin{multline}\label{MaxVacFullPPNmmmffffffiuiuhjuughbghhiuijghghghhhfhhghghguygtjuuujjjkhjhjh11kklkloiouu3333int}
G\left(\vec x,t\right):=\\
\left(\frac{c^2_0}{2\left(\frac{\partial S}{\partial t}+\vec {\tilde
u}\cdot\nabla_{\vec x}S\right)}\left(\Delta_{\vec
x}S-\frac{\partial}{\partial
t}\left\{\frac{1}{c^2_0}\left(\frac{\partial S}{\partial t}+\vec
{\tilde u}\cdot\nabla_{\vec x}S\right)\right\}-\Div_{\vec x}
\left\{\frac{1}{c^2_0}\left(\frac{\partial S}{\partial t}+\vec
{\tilde u}\cdot\nabla_{\vec x} S\right)\vec {\tilde
u}\right\}\right)\right)(\vec x,t),
\end{multline}
we clearly rewrite
\er{MaxVacFullPPNmmmffffffiuiuhjuughbghhiuijghghghhjhjhjjuiuiuiuuiuiint}
as:
\begin{multline}\label{MaxVacFullPPNmmmffffffiuiuhjuughbghhiuijghghghhhfhhghghguygtjjjkkjjkkhhklklkjk33hhjjklikint}
\frac{c^2_0}{2\tilde\omega}\left\{\frac{\partial}{\partial
t}\left\{\frac{1}{c^2_0}\left(\frac{\partial A}{\partial t}+\vec
{\tilde u}\cdot\nabla_{\vec x}A\right)\right\}+\Div_{\vec x}
\left\{\frac{1}{c^2_0}\left(\frac{\partial A}{\partial t}+\vec
{\tilde u}\cdot\nabla_{\vec x} A\right)\vec {\tilde
u}\right\}-\Delta_{\vec x}A\right\}\\+i\left(\frac{\partial
A}{\partial t}
+\vec h\cdot\nabla_{\vec x}A
-G\, A\right)=0
,
\end{multline}
and
\er{MaxVacFullPPNmmmffffffiuiuhjuughbghhiuijghghghhjhjhhghyuyjjjhhjhjffhjjhjojint},
as
\begin{equation}\label{MaxVacFullPPNmmmffffffiuiuhjuughbghhiuijghghghhhgjgfghjdfjgddg1133khhujuuuo99klklljljljhjhjhyf11int}
\left|\vec k(\vec x,t)\right|^2=1\quad\quad\text{or
equivalently}\quad\quad \left|\vec h(\vec x,t)-\vec {\tilde u}(\vec
x,t)\right|^2=c^2_0(\vec x,t)\,.
\end{equation}
Furthermore, from now assume the case of \underline{delicate}
Geometric Optics approximation. Remind that the delicate Geometric
Optics approximation means the following: assume that the changes in
time of $c_0$, $\vec {\tilde u}$, $A$ and $S$ become essential after
certain interval of time $T_e$ and the changes in space of $c_0$,
$A$ and $S$ become essential in the spatial landscape $L_e$. Then we
assume that
\begin{equation}\label{slochangGGaaffgfgjhjggh99int}
k^2_0c^2_0T^2_e\,\gg\, 1\quad\text{and}\quad k^2_0L^2_e\,\gg\, 1.
\end{equation}
Moreover, as before, we assume that the order of $c_0$ is less or
equal to the order of $c$. Then, using estimation
\er{slochangGGaaffgfgjhjggh99int} we rewrite
\er{MaxVacFullPPNmmmffffffiuiuhjuughbghhiuijghghghhhfhhghghguygtjjjkkjjkkhhklklkjk33hhjjklikint}
as:
\begin{multline}\label{MaxVacFullPPNmmmffffffiuiuhjuughbghhiuijghghghhhfhhghghguygtjjjkkjjkkhhklklkjk33hhjjklikjjkjkjjjkoi;;lll1jkkjkjjljjjkjjkjkjjjkkiiihhint}
-\frac{c_0}{2\tilde\omega}\left(\nabla_{\vec x}A-\left(\vec
k\cdot\nabla_{\vec x}A\right)\vec k\right)\cdot\left(\left(d_{\vec
x}\vec {\tilde u}+\left\{d_{\vec x}\vec {\tilde
u}\right\}^T\right)\cdot\vec k-\left(\left(\vec
k\cdot\left(\left\{d_{\vec x}\vec {\tilde u}\right\}^T+d_{\vec
x}\vec {\tilde u}\right)\right)\cdot\vec k\right)\vec k\right)\\
+\frac{c^2_0}{2\tilde\omega}A\left(\frac{\partial}{\partial
t}\left\{\frac{1}{c^2_0}G\right\}+\Div_{\vec x}
\left\{\frac{1}{c^2_0}G\vec {\tilde
u}\right\}+\frac{1}{c^2_0}(c_0\vec k)\cdot\nabla_{\vec
x}G+\frac{1}{c^2_0}G^2\right) \\-\frac{c^2_0}{2}\left(\Div_{\vec
x}\left\{\frac{1}{\tilde\omega}\left(\nabla_{\vec x}A-\left(\vec
k\cdot\nabla_{\vec x}A\right)\vec k\right)\right\}\right) \approx
-i\left\{\frac{\partial A}{\partial t}
+\vec h\cdot\nabla_{\vec x}A
+\frac{1}{2}\,\left(\Div_{\vec x}\vec h \right)A\right\}\\-
iA\frac{1}{2c_0}\left(c_0\vec k\cdot\nabla_{\vec
x}c_0\right)+iA\frac{1}{2c_0}\left(\frac{\partial c_0}{\partial
t}+\vec h\cdot\nabla_{\vec x}c_0\right)+ \frac{iA}{4}\left(\vec
k\cdot\left(\left\{d_{\vec x}\vec {\tilde u}\right\}^T+d_{\vec
x}\vec {\tilde u}\right)\right)\cdot\vec k
\end{multline}
(see subsection \ref{gjgggggggggggggggggggghkhk} for the details).
Moreover, as before, we can easily deduce that
\er{MaxVacFullPPNmmmffffffiuiuhjuughbghhiuijghghghhhfhhghghguygtjjjkkjjkkhhklklkjk33hhjjklikjjkjkjjjkoi;;lll1jkkjkjjljjjkjjkjkjjjkkiiihhint}
invariant under the change of non-inertial cartesian coordinate
system, in the case that under such change we have $A'=A$ and
$S'=S$.

In particular, if in some Cartesian coordinate system we assume time
independent settings of the problem so that $\frac{\partial
S}{\partial t}=c$, $\frac{\partial c_0}{\partial t}=0$ and
$\frac{\partial \vec {\tilde u}}{\partial t}=0$, then we also have
$\frac{\partial {\vec k}}{\partial t}=\frac{\partial  {\vec
h}}{\partial t}=0$ and $\frac{\partial \tilde\omega}{\partial
t}=\frac{\partial G}{\partial t}=0$, and therefore, time-independent
solutions $A:=A(\vec x)$, of
\er{MaxVacFullPPNmmmffffffiuiuhjuughbghhiuijghghghhhfhhghghguygtjjjkkjjkkhhklklkjk33hhjjklikjjkjkjjjkoi;;lll1jkkjkjjljjjkjjkjkjjjkkiiihhint}
are admitted and they satisfy
\begin{multline}\label{MaxVacFullPPNmmmffffffiuiuhjuughbghhiuijghghghhhfhhghghguygtjjjkkjjkkhhklklkjk33hhjjklikjjkjkjjjkoi;;lll1jkkjkjjljjjkjjkjkjjjkkiiihhukukuint}
-\frac{c_0}{2\tilde\omega}\left(\nabla_{\vec x}A-\left(\vec
k\cdot\nabla_{\vec x}A\right)\vec k\right)\cdot\left(\left(d_{\vec
x}\vec {\tilde u}+\left\{d_{\vec x}\vec {\tilde
u}\right\}^T\right)\cdot\vec k-\left(\left(\vec
k\cdot\left(\left\{d_{\vec x}\vec {\tilde u}\right\}^T+d_{\vec
x}\vec {\tilde u}\right)\right)\cdot\vec k\right)\vec k\right)\\
+\frac{c^2_0}{2\tilde\omega}A\left(\Div_{\vec x}
\left\{\frac{1}{c^2_0}G\vec {\tilde
u}\right\}+\frac{1}{c^2_0}(c_0\vec k)\cdot\nabla_{\vec
x}G+\frac{1}{c^2_0}G^2\right) -\frac{c^2_0}{2}\left(\Div_{\vec
x}\left\{\frac{1}{\tilde\omega}\left(\nabla_{\vec x}A-\left(\vec
k\cdot\nabla_{\vec x}A\right)\vec k\right)\right\}\right)\\ \approx
-i\left\{\vec h\cdot\nabla_{\vec x}A
+\frac{1}{2}\,\left(\Div_{\vec x}\vec h \right)A\right\}\\-
iA\frac{1}{2c_0}\left(c_0\vec k\cdot\nabla_{\vec
x}c_0\right)+iA\frac{1}{2c_0}\left(\vec h\cdot\nabla_{\vec
x}c_0\right)+ \frac{iA}{4}\left(\vec k\cdot\left(\left\{d_{\vec
x}\vec {\tilde u}\right\}^T+d_{\vec x}\vec {\tilde
u}\right)\right)\cdot\vec k.
\end{multline}

Finally, from now we assume that $\vec {\tilde u}$ and $c_0$ vary
sufficiently slowly in space and time variables, so that the
following approximations are valid:
\begin{equation}\label{MaxVacFullPPNmmmffffffhhtygghGGGGyuhggghghghffgiooiiojhjjhhjuiuiugughhjggjjhkint}
\frac{|\nabla_{\vec x}c_0|+\frac{1}{c_0}\left|\frac{\partial
c_0}{\partial t}+\vec {\tilde u}\cdot\nabla_{\vec x}c_0\right|}{
c_0}+\frac{\left| d_{\vec x}\vec {\tilde u}+\left\{d_{\vec x}\vec
{\tilde u}\right\}^T\right|}{
c_0}\ll\frac{\left(| \nabla_{\vec
x}A|+\frac{1}{c_0}\left|\frac{\partial A}{\partial t}+\vec {\tilde
u}\cdot\nabla_{\vec x}A\right|\right)}{|A|},
\end{equation}
In particular, if in some Cartesian coordinate system we have
\begin{equation}\label{MaxVacFullPPNmmmffffffiuiuhjuughbghhiuijghghghhjhjhhghyuyhhjhjihikjjjjkjljklghhggghint}
\begin{cases}
\frac{\left| d_{\vec x}\vec {\tilde u}+\left\{d_{\vec x}\vec {\tilde
u}\right\}^T\right|^2}{
|\vec {\tilde u}|^2}\ll\frac{\left(| \nabla_{\vec
x}A|+\frac{1}{c_0}\left|\frac{\partial A}{\partial t}+\vec {\tilde
u}\cdot\nabla_{\vec x}A\right|\right)^2}{|A|^2}\\
\frac{\left|\vec {\tilde u}\right|^2}{
c^2_0}\ll 1\\
\frac{|\nabla_{\vec x}c_0|+\frac{1}{c_0}\left|\frac{\partial
c_0}{\partial t}+\vec {\tilde u}\cdot\nabla_{\vec x}c_0\right|}{
c_0}\ll\frac{\left(| \nabla_{\vec
x}A|+\frac{1}{c_0}\left|\frac{\partial A}{\partial t}+\vec {\tilde
u}\cdot\nabla_{\vec x}A\right|\right)}{|A|},
\end{cases}
\end{equation}
then
\er{MaxVacFullPPNmmmffffffhhtygghGGGGyuhggghghghffgiooiiojhjjhhjuiuiugughhjggjjhkint}
indeed holds! Then, by
\er{MaxVacFullPPNmmmffffffhhtygghGGGGyuhggghghghffgiooiiojhjjhhjuiuiugughhjggjjhkint}
we rewrite and simplify
\er{MaxVacFullPPNmmmffffffiuiuhjuughbghhiuijghghghhhfhhghghguygtjjjkkjjkkhhklklkjk33hhjjklikjjkjkjjjkoi;;lll1jkkjkjjljjjkjjkjkjjjkkiiihhint}
as
\begin{multline}\label{MaxVacFullPPNmmmffffffiuiuhjuughbghhiuijghghghhhfhhghghguygtjjjkkjjkkhhklklkjk33hhjjklikjjkjkjjjkoi;;lll1jkkjkjjljjjkjjkjkjjjkkiiihhihhjhint}
\frac{1}{2\tilde\omega}A\left(\frac{\partial G}{\partial
t}+\Div_{\vec x} \left\{G\vec {\tilde u}\right\}+(c_0\vec
k)\cdot\nabla_{\vec x}G+G^2\right)
\\-\frac{1}{2}\left(\Div_{\vec
x}\left\{\frac{c^2_0}{\tilde\omega}\left(\nabla_{\vec x}A-\left(\vec
k\cdot\nabla_{\vec x}A\right)\vec k\right)\right\}\right) \approx
-i\left\{\frac{\partial A}{\partial t}
+\vec h\cdot\nabla_{\vec x}A
+\frac{1}{2}\,\left(\Div_{\vec x}\vec h \right)A\right\}
\end{multline}
(see subsection \ref{gjgggggggggggggggggggghkhk} for the details).
Moreover, as before, we can easily deduce that
\er{MaxVacFullPPNmmmffffffiuiuhjuughbghhiuijghghghhhfhhghghguygtjjjkkjjkkhhklklkjk33hhjjklikjjkjkjjjkoi;;lll1jkkjkjjljjjkjjkjkjjjkkiiihhihhjhint}
invariant under the change of non-inertial cartesian coordinate
system, in the case that under such change we have $A'=A$ and
$S'=S$.

Furthermore, in time independent settings we rewrite
\er{MaxVacFullPPNmmmffffffiuiuhjuughbghhiuijghghghhhfhhghghguygtjjjkkjjkkhhklklkjk33hhjjklikjjkjkjjjkoi;;lll1jkkjkjjljjjkjjkjkjjjkkiiihhukukuint}
as
\begin{multline}\label{MaxVacFullPPNmmmffffffiuiuhjuughbghhiuijghghghhhfhhghghguygtjjjkkjjkkhhklklkjk33hhjjklikjjkjkjjjkoi;;lll1jkkjkjjljjjkjjkjkjjjkkiiihhihhjhijjkint}
\frac{1}{2\tilde\omega}A\left(\Div_{\vec x} \left\{G\vec {\tilde
u}\right\}+(c_0\vec k)\cdot\nabla_{\vec x}G+G^2\right)
\\-\frac{1}{2}\left(\Div_{\vec
x}\left\{\frac{c^2_0}{\tilde\omega}\left(\nabla_{\vec x}A-\left(\vec
k\cdot\nabla_{\vec x}A\right)\vec k\right)\right\}\right) \approx
-i\left\{\vec h\cdot\nabla_{\vec x}A
+\frac{1}{2}\,\left(\Div_{\vec x}\vec h \right)A\right\}.
\end{multline}
In particular, in the case $\vec {\tilde u}=0$ and constant $\vec
k$, $c_0$ and $\tilde\omega$,
\er{MaxVacFullPPNmmmffffffiuiuhjuughbghhiuijghghghhhfhhghghguygtjjjkkjjkkhhklklkjk33hhjjklikjjkjkjjjkoi;;lll1jkkjkjjljjjkjjkjkjjjkkiiihhihhjhijjkint}
simplifies as the paraxial approximation of the Helmholtz equation,
with respect to direction $(-\vec k)$, which is common in Optics:
\begin{equation}\label{MaxVacFullPPNmmmffffffiuiuhjuughbghhiuijghghghhhfhhghghguygtjjjkkjjkkhhklklkjk33hhjjklikjjkjkjjjkoi;;lll1jkkjkjjljjjkjjkjkjjjkkiiihhihhjhijjkjkjkjkklkint}
\frac{c_0}{2\tilde\omega}\left(\Div_{\vec
x}\left\{\left(\nabla_{\vec x}A-\left(\vec h\cdot\nabla_{\vec
x}A\right)\vec h\right)\right\}\right) \approx
-i\vec k\cdot\nabla_{\vec x}A .
\end{equation}

Next, note that, the time dependent equation
\er{MaxVacFullPPNmmmffffffiuiuhjuughbghhiuijghghghhhfhhghghguygtjjjkkjjkkhhklklkjk33hhjjklikjjkjkjjjkoi;;lll1jkkjkjjljjjkjjkjkjjjkkiiihhihhjhint}
is very similar to the Schr\"{o}dinger equation and similarly can be
written in the form
\begin{equation}\label{MaxVacFullPPNmmmffffffiuiuhjuughbghhiuijghghghhhfhhghghguygtjjjkkjjkkhhklklkjk33hhjjklikjjkjkjjjkoi;;lll1jkkjkjjljjjkjjkjkjjjkkiiihhihhjhuouoint}
 i\hbar\frac{\partial A}{\partial t} =\hat H\cdot A.
\end{equation}
where
\begin{multline}\label{MaxVacFullPPNmmmffffffiuiuhjuughbghhiuijghghghhhfhjkkhkhghghguygtjjjkkjjkkhhklklkjk33hhjjklikjjkjkjjjkoi;;lll1jkkjkjjljjjkjjkjkjjjkkiiihhihhjhkklhkhkint}
\hat H\cdot A:= \frac{\hbar}{2}\left(\Div_{\vec
x}\left\{\frac{c^2_0}{\tilde\omega}\left(\nabla_{\vec x}A-\left(\vec
k\cdot\nabla_{\vec x}A\right)\vec k\right)\right\}\right)
-i\hbar\left\{
\vec h\cdot\nabla_{\vec x}A
+\frac{1}{2}\,\left(\Div_{\vec x}\vec h \right)A\right\}\\
-\frac{\hbar}{2\tilde\omega}\left(\frac{\partial G}{\partial
t}+\Div_{\vec x} \left\{G\vec {\tilde u}\right\}+(c_0\vec
k)\cdot\nabla_{\vec x}G+G^2\right)A
\end{multline}
is a self-adjoint linear operator on the complex Hilbert space.

Finally, solving either
\er{MaxVacFullPPNmmmffffffiuiuhjuughbghhiuijghghghhhfhhghghguygtjjjkkjjkkhhklklkjk33hhjjklikjjkjkjjjkoi;;lll1jkkjkjjljjjkjjkjkjjjkkiiihhint}
or
\er{MaxVacFullPPNmmmffffffiuiuhjuughbghhiuijghghghhhfhhghghguygtjjjkkjjkkhhklklkjk33hhjjklikjjkjkjjjkoi;;lll1jkkjkjjljjjkjjkjkjjjkkiiihhihhjhint}
and inserting the solution into
\er{MaxVacFullPPNmmmffffffiuiuhjuughbghhuiiujjjjjjjjgjjgiuiuuiint}
we established $U(\vec x,t)$ in
\er{MaxVacFullPPNmmmffffffiuiuhjuughbghhjjjoojouiuiuint} in the
\underline{delicate} Geometric Optics approximation
\er{slochangGGaaffgfgjhjgghint}.

\subsubsection{Wave optics inside a moving medium and/or in the
presence of gravitational field}\label{jghggghghghint} Assume that
in some inertial or non-inertial cartesian coordinate system a
\underline{complex valued} scalar field $U:=U(\vec x,t)$,
characterizing some wave, satisfies the wave equation
\er{MaxVacFullPPNmmmffffffiuiuhjuughbghhjjint}:
\begin{equation}\label{MaxVacFullPPNmmmffffffiuiuhjuughbghhjjhh11int}
\frac{\partial}{\partial
t}\left\{\frac{1}{c^2_0}\left(\frac{\partial U}{\partial t}+\vec
{\tilde u}\cdot\nabla_{\vec x}U\right)\right\}+\Div_{\vec x}
\left\{\frac{1}{c^2_0}\left(\frac{\partial U}{\partial t}+\vec
{\tilde u}\cdot\nabla_{\vec x} U\right)\vec {\tilde
u}\right\}-\Delta_{\vec x}U=0,
\end{equation}
where, as before, $\vec {\tilde u}:=\vec {\tilde u}(\vec x,t)$ is
some moderately changing (in space and in time) speed-like vector
field and $c_0:=c_0(\vec x,t)>0$ is a moderately changing scalar
quantity, that we call wave propagation speed. In particular, if we
assume that the fields $\vec {\tilde u}$ and $c_0$ are independent
on the time variable, and we consider the case of a monochromatic
wave of the constant frequency $\nu=\frac{\omega}{2\pi}$:
\begin{equation}\label{MaxVacFullPPNmmmffffffiuiuhjuughbghhiuijghghghhjhjhhghyuyhhjhjihikjjjjkjljklhh11int}
U(\vec x,t):=\mathcal{U}(\vec x)e^{i\omega t}\quad\quad\forall(\vec
x,t),
\end{equation}
where $\mathcal{U}:=\mathcal{U}(\vec x)$ is a complex-valued scalar
field, independent of time, then inserting
\er{MaxVacFullPPNmmmffffffiuiuhjuughbghhiuijghghghhjhjhhghyuyhhjhjihikjjjjkjljklhh11int}
into \er{MaxVacFullPPNmmmffffffiuiuhjuughbghhjjhh11int} gives the
following time independent equation:
\begin{equation}\label{MaxVacFullPPNmmmffffffiuiuhjuughbghhjjuggk;klklkljklint}
\frac{i\omega}{c^2_0(\vec x)}\left(i\omega \mathcal{U}(\vec x)+\vec
{\tilde u}(\vec x)\cdot\nabla_{\vec x}\mathcal{U}(\vec
x)\right)+\Div_{\vec x} \left\{\frac{1}{c^2_0(\vec x)}\left(i\omega
\mathcal{U}(\vec x)+\vec {\tilde u}(\vec x)\cdot\nabla_{\vec x}
\mathcal{U}(\vec x)\right)\vec {\tilde u}(\vec
x)\right\}-\Delta_{\vec x}\mathcal{U}(\vec x)=0.
\end{equation}

Next, suppose that, although we cannot consider the Geometric Optics
approximation \er{slochangGGaaffgfgjhjgghint}, we can consider,
however, the following assumption:
\begin{equation}\label{iggguggkhhhhint}
\vec {\tilde u}(\vec x,t)=\nabla_{\vec x}\tilde Z(\vec x,t)
\quad\quad\forall(\vec x,t),
\end{equation}
where $\tilde Z:=\tilde Z(\vec x,t)$ is some real-valued scalar
function, which satisfies
\begin{equation}\label{iggguggkhhhh22int}
\frac{1}{c^2_0}\left|\frac{\partial \tilde Z}{\partial
t}\right|+\frac{1}{c^2_0}\left|\nabla_{\vec x}\tilde Z\right|^2\ll
1.
\end{equation}
Note that \er{iggguggkhhhhint} is valid in the particular case of
the electromagnetic waves, propagating in vacuum, in the
\underline{inertial} or-more generally non-rotating cartesian
coordinate system (where we indeed have $\vec {\tilde u}=\vec v$ and
$curl_{\vec x}\vec v=0$). On the other hand, \er{iggguggkhhhh22int}
means that $\frac{\left|\vec {\tilde u}\right|^2}{c^2_0}\ll 1$ and
$\vec {\tilde u}$ is either independent of time or depends on time
slowly. Moreover, assume that $c_0$ is a constant in the given
region.

Next we do the change of variables $(\vec x,t)\to(\vec z,\tau)$ as
\begin{equation}\label{hvkgkjgkjkhhhhjjkhkhkkjjjjhgjint}
\begin{cases}
\vec z:=\vec x,\\
\tau:=t+\frac{1}{c^2_0}\tilde Z(\vec x,t),
\end{cases}
\end{equation}
so that, we can find a complex-valued function $V:=V(\vec z,\tau)$
depending on  $\vec z\in\mathbb{R}^3$ and a real $\tau$, which
satisfies:
\begin{equation}\label{hvkgkjgkjkhhhhjjkhkhkint}
U(\vec x,t):=V\left(\vec x,t+\frac{1}{c^2_0}\tilde Z(\vec
x,t)\right).
\end{equation}
Then in subsection \ref{jghggghghgh} we deduce that the function
$V(\vec z,\tau)$ approximately solves the standard type of the wave
equation of the form:
\begin{equation}\label{MaxVacFullPPNmmmffffffiuiuhjuughbghhjjkhhkhhhhhjint}
\frac{1}{c^2_0}\frac{\partial^2 V}{\partial\tau^2}\left(\vec
z,\tau\right)-\Delta_{\vec z}V\left(\vec z,\tau\right)=0.
\end{equation}
Fortunately, the plenty of tools for analytical resolution of the
\underline{standard} wave equation
\er{MaxVacFullPPNmmmffffffiuiuhjuughbghhjjkhhkhhhhhjint} is well
known. In particular, if we assume that the fields $\vec {\tilde u}$
and $c_0$ are independent of the time variable, and we consider the
case of a monochromatic wave of the constant frequency
$\nu=\frac{\omega}{2\pi}$, as in
\er{MaxVacFullPPNmmmffffffiuiuhjuughbghhiuijghghghhjhjhhghyuyhhjhjihikjjjjkjljklhh11int},
then, we can choose $\tilde Z$ to be independent of time and
moreover, inserting
\er{MaxVacFullPPNmmmffffffiuiuhjuughbghhiuijghghghhjhjhhghyuyhhjhjihikjjjjkjljklhh11int}
into \er{hvkgkjgkjkhhhhjjkhkhkint} gives:
\begin{equation}\label{hvkgkjgkjkhhhhjjkhkhkhkhhjhgjkjiuiiyjljlint}
V\left(\vec z,\tau\right)=\mathcal{V}\left(\vec z\right)e^{i\omega
\tau}\quad\quad\text{and}\quad\quad \mathcal{U}(\vec
x):=\mathcal{V}\left(\vec x\right)e^{\frac{i\omega }{c^2_0(\vec
x)}\tilde Z(\vec x)},
\end{equation}
where $\mathcal{V}:=\mathcal{V}(\vec z)$ is a complex-valued scalar
field, independent of time. Moreover, inserting
\er{hvkgkjgkjkhhhhjjkhkhkhkhhjhgjkjiuiiyjljlint} into
\er{MaxVacFullPPNmmmffffffiuiuhjuughbghhjjkhhkhhhhhjint} gives that
$\mathcal{V}(\vec x)$ solves the Helmholtz equation of the form:
\begin{equation}\label{MaxVacFullPPNmmmffffffiuiuhjuughbghhjjuggk;klklkljklihihjhgjggint}
\frac{\omega^2}{c^2_0}\mathcal{V}(\vec x)+\Delta_{\vec
x}\mathcal{V}(\vec x)=0,
\end{equation}
and
\begin{equation}\label{MaxVacFullPPNmmmffffffiuiuhjuughbghhjjuggk;klklkljklihihjhgjggjkijjhjhhint}
U(\vec x,t)=\mathcal{V}\left(\vec
x\right)e^{i\omega\left(t+\frac{1}{c^2_0}\tilde Z(\vec x)\right)}.
\end{equation}
Thus, if, as before in
\er{MaxVacFullPPNmmmffffffiuiuhjuughbghhuiiujjjjjjjjhhhjjjkkint}, we
denote $k_0:=\frac{\omega}{c}$, where $c$ is a constant in the
Maxwell equations for the vacuum, and express $U(\vec x,t)$ as in
\er{MaxVacFullPPNmmmffffffiuiuhjuughbghhuiiujjint} by the following:
\begin{equation}\label{MaxVacFullPPNmmmffffffiuiuhjuughbghhuiiujjuyughghint}
U(\vec x,t)=A(\vec x,t)e^{ik_0S(\vec x,t)}\,,
\end{equation}
where $A:=A(\vec x,t)$ and $S:=S(\vec x,t)$ are real scalar fields,
then by
\er{MaxVacFullPPNmmmffffffiuiuhjuughbghhjjuggk;klklkljklihihjhgjggjkijjhjhhint}
we deduce:
\begin{equation}\label{MaxVacFullPPNmmmffffffiuiuhjuughbghhuiiujjuyughghioouiint}
U(\vec x,t)=A(\vec x,t)e^{ik_0S(\vec x,t)},\quad\text{where}\quad
A(\vec x,t)=B(\vec x)\quad\text{and}\quad S(\vec x,t)=ct+P(\vec
x)+\frac{c}{c^2_0}\tilde Z(\vec x),
\end{equation}
with
\begin{equation}\label{MaxVacFullPPNmmmffffffiuiuhjuughbghhjjuggk;klklkljklihihjhgjggjkijjhjhhjoojjioioint}
B(\vec x)e^{ik_0P(\vec x)}:=\mathcal{V}\left(\vec x\right),
\end{equation}
where $B:=B(\vec x)$ and $P:=P(\vec x)$ are real scalar fields and
$\mathcal{V}\left(\vec x\right)$ is a complex solution of
\er{MaxVacFullPPNmmmffffffiuiuhjuughbghhjjuggk;klklkljklihihjhgjggint}.
In particular, if the vector field $\vec h_0$ is the direction of
the propagation of the ray that passes through point $\vec x$,
defined for every $\vec x$, as in
\er{MaxVacFullPPNmmmffffffiuiuhjuughbghhiuijghghghhhfhhghghguygtjuuujjjkyuuyint}
by:
\begin{equation}\label{MaxVacFullPPNmmmffffffiuiuhjuughbghhiuijghghghhhfhhghghguygtjuuujjjkyuuyjhjhjhint}
\vec h_0(\vec x):=\frac{c}{c^2_0}\vec {\tilde u}(\vec
x)-\nabla_{\vec x}S(\vec x),
\end{equation}
then by inserting
\er{MaxVacFullPPNmmmffffffiuiuhjuughbghhuiiujjuyughghioouiint} and
\er{iggguggkhhhhint} into
\er{MaxVacFullPPNmmmffffffiuiuhjuughbghhiuijghghghhhfhhghghguygtjuuujjjkyuuyjhjhjhint}
gives
\begin{equation}\label{MaxVacFullPPNmmmffffffiuiuhjuughbghhiuijghghghhhfhhghghguygtjuuujjjkyuuyjhjhjhjhjhjhhjhint}
\vec h_0(\vec x)=
-\nabla_{\vec x}P(\vec x).
\end{equation}
As $P(\vec x)$ being the phase of $\mathcal{V}(\vec x)$ is
independent of the field $\vec {\tilde u}$ we obtain that the
direction of the ray $\vec h$, in the general case of nontrivial
field $\vec {\tilde u}$, is the same as in the case  $\vec {\tilde
u}(\vec x)\equiv 0$. So the nontrivial motion of the medium and the
nontrivial vectorial gravitational potential dose not affect the
direction of the ray propagation $\vec h$, provided we assume
\er{iggguggkhhhhint} and \er{iggguggkhhhh22int}. Moreover, if we
denote by $U(\vec x,t)=A(\vec x,t)e^{ik_0S(\vec x,t)}$ the
monochromatic wave solution of the problem in the general case of
nontrivial field $\vec {\tilde u}$ and by $U_0(\vec x,t)=A_0(\vec
x,t)e^{ik_0S_0(\vec x,t)}$ the monochromatic wave solution of the
corresponding simpler problem in the case of the trivial field $\vec
{\tilde u}\equiv 0$ then we have
\begin{equation}\label{MaxVacFullPPNmmmffffffiuiuhjuughbghhuiiujjuyughghioouijkkjhjint}
U(\vec x,t)=U_0(\vec x,t)e^{\frac{i\omega}{c^2_0}\tilde Z(\vec
x)},\quad A(\vec x,t)=A_0(\vec x,t)\quad\text{and}\quad S(\vec
x,t)=S_0(\vec x,t)+\frac{c}{c^2_0}\tilde Z(\vec x),
\end{equation}
i.e. in the case, where \er{iggguggkhhhhint} and
\er{iggguggkhhhh22int} hold, the wave solution for the general case
of nontrivial $\vec {\tilde u}$ can be obtained from the
corresponding solution in the simplest trivial case by simple adding
the term $\frac{\omega}{c^2_0}\tilde Z(\vec x)$ to the phase. In
particular, all the pictures of the interference or the diffraction
for given nontrivial $\vec {\tilde u}$ will be the same as in the
trivial case, provided we assume \er{iggguggkhhhhint} and
\er{iggguggkhhhh22int}. Finally, by simple adding the phase
$\frac{\omega}{c^2_0}\tilde Z(\vec x)$ we obtain the particular
monochromatic solutions of
\er{MaxVacFullPPNmmmffffffiuiuhjuughbghhjjhh11int} from the
following well known particular monochromatic solutions of the
simplest wave equation (i.e. for the case $\vec {\tilde u}\equiv
0$):
\begin{itemize}
\item
plane wave,
\item
spherical wave,
\item
the approximate solution in the form of the Gaussian beam.
\end{itemize}

\section{Notations and preliminaries} \noindent$\bullet$ By
$\R^{p\times q}$ we denote the set of $p\times q$-matrixes with real
coefficients.



\noindent$\bullet$ For a $p\times q$ matrix $A$ with $ij$-th entry
$a_{ij}$ and for a $q\times d$ matrix $B$ with $ij$-th entry
$b_{ij}$ we denote by $AB:=A\cdot B$ their product, i.e. the
$p\times d$ matrix, with $ij$-th entry
$\sum\limits_{k=1}^{q}a_{ik}b_{kj}$.

\noindent$\bullet$ We identify a vector $\vec
u=(u_1,\ldots,u_q)\in\R^q$ with the $q\times 1$ matrix
having $i1$-th entry $u_i$, so that for the $p\times q$ matrix $A$
with $ij$-th entry $a_{ij}$ and for $\vec
v=(v_1,v_2,\ldots,v_q)\in\R^q$ we denote by $A\,\vec v :=A\cdot\vec
v$ the $p$-dimensional vector $\vec u=(u_1,\ldots,u_p)\in\R^p$,
given by $u_i=\sum\limits_{k=1}^{q}a_{ik}v_k$ for every $1\leq i\leq
p$.

\noindent$\bullet$ For a $p\times q$ matrix $A$ with $ij$-th entry
$a_{ij}$ denote by $A^T$ the transpose $q\times p$ matrix with
$ij$-th entry $a_{ji}$.

\noindent$\bullet$ For a $p\times p$ matrix $A$ with $ij$-th entry
$a_{ij}$ denote $tr(A):=\sum_{k=1}^{p}a_{kk}$ (the trace of the
matrix $A$).

\noindent$\bullet$ For a $p\times q$ real matrices $A$ and $B$ with
$ij$-th entries $a_{ij}$ and $b_{ij}$ denote by $A:B$ their scalar
product $A:B:=tr(A\cdot
B^T):=\sum_{j=1}^{p}\sum_{k=1}^{q}a_{jk}b_{jk}$

\noindent $\bullet$ For a $p\times p$ real matrix $A$ with $ij$-th
entry $a_{ij}$ denote by $|A|$ its standard norm
$|A|:=\sqrt{A:A}=\sqrt{\sum_{j=1}^{p}\sum_{k=1}^{p}a^2_{jk}}$.

\noindent$\bullet$ For $\vec u=(u_1,\ldots,u_p)\in\R^p$ and $\vec
v=(v_1,\ldots,v_p)\in\R^p$ we denote by $\vec u\vec v:=\vec
u\cdot\vec v:=\sum\limits_{k=1}^{p}u_k v_k$ the standard scalar
product. We also note that $\vec u\vec v=\vec u^T\vec v=\vec v^T\vec
u$ as products of matrices.

\noindent$\bullet$ For $\vec u=(u_1,u_2,u_3)\in\R^3$ and $\vec
v=(v_1,v_2,v_3)\in\R^3$ we denote
$$\vec u\times\vec v:=
\left(u_2 v_3-u_3 v_2, u_3 v_1-u_1 v_3, u_1 v_2-u_2
v_1\right)\in\R^3.$$

\noindent$\bullet$ For $\vec u=(u_1,\ldots,u_p)\in\R^p$ and $\vec
v=(v_1,\ldots,v_q)\in\R^q$ we denote by $\vec u\otimes\vec v$ the
$p\times q$ matrix with $ij$-th entry $u_i v_j$ (i.e. $\vec
u\otimes\vec v=\vec u\,\vec v^T$ as a product of matrices).

\noindent$\bullet$
Given a vector valued function $\vec f(\vec x)=\left(f_1(\vec
x),\ldots,f_k(\vec x)\right):\O\to\R^k$ ($\O\subset\R^N$) we denote
by $D\vec f$ the $k\times N$ matrix with $ij$-th entry
$\frac{\partial f_i}{\partial x_j}$. In the case of a scalar valued
function $\psi(\vec x):\O\subset\R^N\to\R$ we associate with $D
\psi$ (which, by definition, belongs to $\R^{1\times N}$) the
corresponding vector $\nabla \psi:=\left(\frac{\partial
\psi}{\partial x_1},\ldots,\frac{\partial \psi}{\partial
x_N}\right)$.

\noindent$\bullet$ Given a matrix valued function $F(\vec
x):=\{F_{ij}(\vec x)\}:\O\subset\R^N\to\R^{k\times N}$,
we denote by ${div}\, F$ the $\R^k$-valued vector field defined by
$div\,F(\vec x):=(l_1,\ldots,l_k)(\vec x)$ where $l_i(\vec
x)=\sum\limits_{j=1}^{N}\frac{\partial F_{ij}}{\partial x_j}(\vec
x)$. Given a vector valued function $\vec f(\vec
x):=\left(f_{1}(\vec x),\ldots, f_{N}(\vec
x)\right):\O\subset\R^N\to\R^{N}$ we denote
${div}\, \vec f:=\sum\limits_{j=1}^{N}\frac{\partial f_j}{\partial
x_j}$.

\noindent$\bullet$ Given a scalar or vector valued function $\vec
f(\vec x):\O\subset\R^N\to\R^{k}$ we denote by $\Delta \vec f$ the
Laplacian of $\vec f$ defined by $\Delta \vec
f:=\sum\limits_{j=1}^{N}\frac{\partial^2\vec  f}{\partial x^2_j}$.

\noindent$\bullet$ Given a vector valued function $\vec f(\vec
x)=\left(f_1(\vec x),f_2(\vec x),f_3(\vec
x)\right):G\subset\R^3\to\R^3$
we denote
$$curl\,\vec f(\vec x):=\left(\frac{\partial f_3}{\partial x_2}-\frac{\partial f_2}{\partial x_3}\,,
\frac{\partial f_1}{\partial x_3}-\frac{\partial f_3}{\partial
x_1}\,,\frac{\partial f_2}{\partial x_1}-\frac{\partial
f_1}{\partial x_2}\right)(\vec x).$$
%
%
%
%
%
%
%
%
%
%
We have the following trivial identities:
\begin{align}
\vec a\times \vec b=-\vec b\times\vec a\quad\text{and}\quad \vec
a\cdot(\vec b\times\vec c)=(\vec a\times\vec b)\cdot\vec
c\quad\quad\forall\,\vec a,\vec b,\vec c\in\R^3,\label{apfrm1}\\
\vec a\times(\vec b\times\vec c)=(\vec a\cdot\vec c)\,\vec b-(\vec
a\cdot\vec b)\,\vec c\quad\quad\forall\,\vec a,\vec b,\vec
c\in\R^3,\label{apfrm2}\\
(A\cdot \vec b)\times \vec c-(A\cdot \vec c)\times \vec
b=tr(A)\,(\vec b\times \vec c)-A^T\cdot(\vec b\times \vec
c)\quad\quad\forall A\in\R^{3\times 3},\;\forall\,\vec b,\vec
c\in\R^3,\label{apfrm10}\\
A^T\cdot\left((A\cdot \vec b)\times (A\cdot \vec c)\right)=(det\,
A)\,(\vec b\times \vec c)\quad\quad\forall A\in\R^{3\times
3},\;\forall\,\vec b,\vec
c\in\R^3,\label{apfrm90}\\
{div}(\vec f\times \vec g)=\vec g\cdot curl\,\vec f-\vec f\cdot
curl\,\vec g\quad\quad
\forall\,\vec f,\vec g:G\subset\R^3\to\R^3,\label{apfrm3}\\
{div}(\psi \vec f)=\psi\,{div}\, \vec f+\nabla\psi\cdot \vec
f\quad\quad
\forall\,\psi:G\subset\R^3\to\R,\;\;
\forall\,\vec f:G\subset\R^3\to\R^3,\label{apfrm4}\\
curl\,(\psi \vec f)=\psi\,curl\, \vec f+\nabla\psi\times \vec
f\quad\quad \forall\,\psi:G\subset\R^3\to\R,\;\;
\forall\,\vec f:G\subset\R^3\to\R^3,\label{apfrm5}\\
{div}\left(curl\,\vec f\right)=0\quad\quad\forall\,\vec
f:G\subset\R^3\to\R^3,\label{apfrm11}\\
{curl}\left(\nabla\psi\right)=0\quad\quad\forall\,\vec
\psi:G\subset\R^3\to\R,\label{apfrm100}
\\ curl\,\left(curl\,\vec f\right)=\nabla\left({div}\, \vec f\right)-\Delta\vec
f\quad\quad\forall\,\vec f:G\subset\R^3\to\R^3,\label{apfrm12}
\\ curl\,(\vec f\times \vec g)=({div}\, \vec g)\,\vec f -({div}\, \vec
f)\,\vec g +(D \vec f)\cdot \vec g-(D \vec g)\cdot \vec f\quad\quad
\forall\,\vec f,\vec g:G\subset\R^3\to\R^3,\label{apfrm6}\\
curl\,(\vec f\times \vec g)=div\left(\vec f\otimes \vec g-\vec
g\otimes \vec f\right)\quad\quad
\forall\,\vec f,\vec g:G\subset\R^3\to\R^3,\label{apfrm6kkk}\\
div\left(\vec f\otimes \vec g\right)=(D \vec f)\cdot \vec g+({div}\,
\vec g)\,\vec f\quad\quad
\forall\,\vec f,\vec g:G\subset\R^3\to\R^3,\label{apfrm6kkkppp}\\
\nabla(\vec f\cdot \vec g)=(D \vec f)^T\cdot \vec g+(D
\vec g)^T\cdot \vec f\quad\quad \forall\,\vec f,\vec g:G\subset\R^3\to\R^3,\label{apfrm7}\\
\vec f\times(curl\,\vec g)=(D \vec g)^T\cdot \vec f-(D
\vec g)\cdot \vec f\quad\quad \forall\,\vec f,\vec g:G\subset\R^3\to\R^3,\label{apfrm9}\\
\nabla(\vec f\cdot \vec g)=\vec f\times(curl\,\vec g)+\vec
g\times(curl\,\vec f)+(D \vec f)\cdot \vec g+(D \vec g)\cdot \vec
f\quad\quad \forall\,\vec f,\vec
g:G\subset\R^3\to\R^3,\label{apfrm8}
\end{align}
where
we mean by $A\cdot \vec l$  the usual product of matrix
$A\in\R^{3\times 3}$ and vector $\vec l\in\R^3$ and by $A^T$ we mean
the transpose of matrix $A$.

\section{Transformations of scalar and vector fields under the change of inertial or non-inertial cartesian coordinate
system}\label{fhfgfghdfdfdfd}
\begin{definition}\label{bggghghgj}
Consider the change of some inertial or non-inertial cartesian
coordinate system $(*)$ to another cartesian coordinate system
$(**)$ of the form:
\begin{equation}\label{noninchgravortbstr}
\begin{cases}
\vec x'=A(t)\cdot\vec x+\vec z(t),\\
t'=t,
\end{cases}
\end{equation}
where $A(t)\in SO(3)$ is a rotation, i.e. $A(t)\in \R^{3\times 3}$,
$det\, A(t)>0$ and $A(t)\cdot A^T(t)=I$.
\begin{itemize}
\item
We say that the scalar field $\psi:=\psi(\vec
x,t):\R^3\times[t_0,+\infty)\to\R$ is a proper scalar field if,
under every change of coordinate system given by
\er{noninchgravortbstr}, this field transforms by the law:
\begin{equation}\label{uguyytfdddd}
\psi'(\vec x',t')=\psi(\vec x,t).
\end{equation}

\item
We say that the vector field $\vec f:=\vec f(\vec
x,t):\R^3\times[t_0,+\infty)\to\R^3$ is a proper vector field if,
under every change of coordinate system given by
\er{noninchgravortbstr}, this field transforms by the law:
\begin{equation}\label{uguyytfddddgghjjg}
\vec f'(\vec x',t')=A(t)\cdot\vec f(\vec x,t),
\end{equation}

\item
We say that the vector field $\vec v:=\vec v(\vec
x,t):\R^3\times[t_0,+\infty)\to\R^3$ is a speed-like vector field
if, under every change of coordinate system given by
\er{noninchgravortbstr}, this field transforms by the law:
\begin{equation}
\label{NoIn5redbstr}\vec v'(\vec x',t')=A(t)\cdot \vec v(\vec x,t)+
\frac{d A}{dt}(t)\cdot\vec x+
\vec w(t),
\end{equation}
where we set
\begin{equation}\label{buitguihjk}
\vec w(t):=\frac{d\vec z}{dt}(t)\quad\quad\forall\,t.
\end{equation}

\item
We say that the matrix valued field $T:=T(\vec
x,t):\R^3\times[t_0,+\infty)\to\R^{3\times 3}$ is a proper matrix
field if, under every change of coordinate system given by
\er{noninchgravortbstr}, this field transforms by the law:
\begin{equation}\label{uguyytfddddgghjjghjjj}
T'(\vec x',t')=A(t)\cdot T(\vec x,t)\cdot A^T(t)=A(t)\cdot T(\vec
x,t)\cdot \left\{A(t)\right\}^{-1}.
\end{equation}
\end{itemize}
\end{definition}
\begin{proposition}\label{yghgjtgyrtrt}
If $\psi:\R^3\times[t_0,+\infty)\to\R$ is a proper scalar field,
$\vec f:\R^3\times[t_0,+\infty)\to\R^3$ and $\vec
g:\R^3\times[t_0,+\infty)\to\R^3$ are proper vector fields, $\vec
v:\R^3\times[t_0,+\infty)\to\R^3$ and $\vec
u:\R^3\times[t_0,+\infty)\to\R^3$ are speed-like vector fields and
$T:\R^3\times[t_0,+\infty)\to\R^{3\times 3}$ is a proper matrix
field, then:
\begin{itemize}
\item[{\bf(i)}] scalar fields defined in every coordinate system as $\,\vec f\cdot\vec g$, $div_{\vec x} \vec f$ and $div_{\vec x} \vec
v$ are proper scalar fields;

\item[{\bf(ii)}] vector fields defined in every coordinate system as $\nabla_{\vec x}\psi$, $div_{\vec x} T$, $curl_{\vec x}\vec
f$, $\vec f\times\vec g$,
$div_{\vec x}\left(d_{\vec x} \vec v+\left\{d_{\vec x} \vec
v\right\}^T\right)$, $\nabla_{\vec x}\left(div_{\vec x}\vec
v\right)$, $\Delta_{\vec x}\vec v$, $curl_{\vec x}\left(curl_{\vec
x}\vec v\right)$ and $(\vec u-\vec v)$ are proper vector fields;

\item[{\bf(iii)}] matrix fields defined in every coordinate system as $\,d_{\vec x} \vec f$ and $\left(d_{\vec x} \vec v+\left\{d_{\vec x} \vec v\right\}^T\right)$ are proper
matrix fields;

\item[{\bf(iv)}] scalar fields $\xi:\R^3\times[t_0,+\infty)\to\R$ and $\zeta:\R^3\times[t_0,+\infty)\to\R$, defined
in every coordinate system by
\begin{equation}\label{vfyutuyfffhhgfhgfhgtg}
\xi:=\frac{\partial\psi}{\partial t}+\vec v\cdot\nabla_{\vec x}\vec
\psi\quad \text{and}\quad \zeta:=\frac{\partial\psi}{\partial
t}+div_{\vec x}\left\{ \psi\vec v\right\}
\end{equation}
are proper scalar fields;

\item[{\bf(v)}] vector fields $\vec\Theta:\R^3\times[t_0,+\infty)\to\R^3$ and $\vec\Xi:\R^3\times[t_0,+\infty)\to\R^3$, defined
in every coordinate system by
\begin{equation}\label{vfyutuyfffhhgfhgfh}
\vec\Theta:=\frac{\partial \vec f}{\partial t}- curl_{\vec
x}\left(\vec v\times \vec f\right)+\left({div}_{\vec x}\vec
f\right)\vec v\quad\text{and}\quad \vec\Xi:=\frac{\partial \vec
f}{\partial t}- \vec v\times curl_{\vec x}\vec f+\nabla_{\vec
x}\left(\vec v\cdot\vec f\right),
\end{equation}
are proper vector fields and
\begin{equation}
\label{vhfffngghhjghhgjlkhjhkPPP} \vec\Xi
=
\vec\Theta-\left(div_{\vec x}\vec v\right)\vec f+ \left(d_{\vec
x}\vec v+\left\{d_{\vec x}\vec v\right\}^T\right)\cdot\vec f.
\end{equation}
\end{itemize}
\end{proposition}
\begin{proof}
By \er{noninchgravortbstr} and the chain rule for every vector
fields $\vec \Gamma:\R^3\times[t_0,+\infty)\to\R^3$ and $\vec
\Lambda:\R^3\times[t_0,+\infty)\to\R^3$ we have
\begin{equation}
\label{vyguiuiujggghjjgredbstr}
\begin{cases}
\left(A(t)\cdot\vec \Gamma\right)\cdot\left(A(t)\cdot\vec
\Lambda\right)=\vec\Gamma\cdot\vec \Lambda\\
\left(A(t)\cdot\vec \Gamma\right)\times\left(A(t)\cdot\vec
\Lambda\right)=A(t)\cdot\left(\vec\Gamma\times\vec \Lambda\right)
\\
d_{\vec x'}\vec \Gamma=\left(d_{\vec x}\vec \Gamma\right)\cdot A^{-1}(t)\\
curl_{\vec x'}\left( A(t)\cdot\vec \Gamma\right)=A(t)\cdot
curl_{\vec x}
\vec \Gamma\\
div_{\vec x'}\left( A(t)\cdot\vec \Gamma\right)=div_{\vec x}\vec
\Gamma.
\end{cases}
\end{equation}
Thus, in particular, by \er{vyguiuiujggghjjgredbstr} and
\er{uguyytfddddgghjjg} we have
\begin{equation}\label{uguyytfddddgghjjguyyjyhkkyhhjhjh}
\vec f'\cdot\vec g'=\vec f\cdot\vec g,\quad\quad \vec f'\times\vec
g'=A(t)\left(\vec f\times\vec g\right),
\end{equation}
and
\begin{equation}\label{uguyytfddddgghjjguyyjyhkkyh}
div_{\vec x'}\vec f'=div_{\vec x'}\left(A(t)\cdot\vec
f\right)=div_{\vec x}\vec f,
\end{equation}
and by \er{vyguiuiujggghjjgredbstr} and \er{NoIn5redbstr} we have
\begin{multline}
\label{NoIn5redbstrgfhgfhfh} div_{\vec x'}\vec v'=div_{\vec
x'}\left\{A(t)\cdot \vec v+
A'(t)\cdot\vec x+\vec w(t)\right\}=div_{\vec x}\left\{\vec
v+A^{-1}(t)\cdot A'(t)\cdot\vec x+A^{-1}(t)\cdot\vec
w(t)\right\}\\=div_{\vec x}\vec v+tr\left(A^{-1}(t)\cdot
A'(t)\right).
\end{multline}
where $tr\left(A^{-1}(t)\cdot A'(t)\right)$ is the trace of the
matrix $A^{-1}(t)\cdot A'(t)$ (sum of diagonal elements). However,
since $A^T(t)\cdot A(t)=I$ we have $A^{-1}(t)=A^T(t)$ and
$A^{-1}(t)\cdot A'(t)=S(t)$, where $S^T(t)=-S(t)$. In particular
$tr\, S(t)=0$ and thus
\begin{equation}
\label{NoIn5redbstrgfhgfhfhuhhjhjj} tr\,\left(A^{-1}(t)\cdot
A'(t)\right)=0.
\end{equation}
Thus by \er{NoIn5redbstrgfhgfhfh} and
\er{NoIn5redbstrgfhgfhfhuhhjhjj} we have
\begin{equation}
\label{NoIn5redbstrgfhgfhfhgffgf} div_{\vec x'}\vec v'=div_{\vec
x}\vec v.
\end{equation}
So by \er{uguyytfddddgghjjguyyjyhkkyhhjhjh},
\er{uguyytfddddgghjjguyyjyhkkyh} and \er{NoIn5redbstrgfhgfhfhgffgf}
we proved {\bf(i)}.

 Next by \er{vyguiuiujggghjjgredbstr} and \er{uguyytfddddgghjjg} we
have
\begin{equation}\label{uguyytfddddgghjjghjg}
d_{\vec x'}\vec f'=d_{\vec x'}\left(A(t)\cdot\vec f\right)=A(t)\cdot
d_{\vec x'}\vec f=A(t)\cdot\left(d_{\vec x}\vec f\right)\cdot
A^{-1}(t)=A(t)\cdot\left(d_{\vec x}\vec f\right)\cdot A^T(t),
\end{equation}
and by \er{vyguiuiujggghjjgredbstr} and \er{NoIn5redbstr} we have
\begin{multline}
\label{NoIn5redbstruiytrr} d_{\vec x'}\vec v'=d_{\vec
x'}\left(A(t)\cdot \vec v+ A'(t) \cdot\vec x+ \vec
w(t)\right)=A(t)\cdot d_{\vec x'}\vec v+d_{\vec x}\left(A'(t)
\cdot\vec x\right)\cdot A^{-1}(t)\\=A(t)\cdot\left(d_{\vec x}\vec
v\right)\cdot A^{-1}(t)+A'(t)A^{-1}(t)=A(t)\cdot\left(d_{\vec x}\vec
v\right)\cdot A^T(t)+A'(t)\cdot A^T(t).
\end{multline}
Then taking the transpose of the both sides of
\er{NoIn5redbstruiytrr} we infer
\begin{equation}
\label{NoIn5redbstruiytrrhjjggh} \left\{d_{\vec x'}\vec
v'\right\}^T=A(t)\cdot\left\{d_{\vec x}\vec v\right\}^T\cdot
A^T(t)+A(t)\cdot \left\{A'(t)\right\}^T.
\end{equation}
However, as before, since $A(t)\cdot A^T(t)=I$ we have $A'(t)\cdot
A^T(t)+A(t)\cdot \left\{A'(t)\right\}^T=0$, by
\er{NoIn5redbstruiytrr} and \er{NoIn5redbstruiytrrhjjggh} we have
\begin{equation}
\label{NoIn5redbstruiytrrhjjgghjkgkgh} \left(d_{\vec x'}\vec
v'+\left\{d_{\vec x'}\vec v'\right\}^T\right)=A(t)\cdot\left(d_{\vec
x}\vec v+\left\{d_{\vec x}\vec v\right\}^T\right)\cdot A^T(t).
\end{equation}
So by \er{uguyytfddddgghjjghjg} and
\er{NoIn5redbstruiytrrhjjgghjkgkgh} we proved {\bf(iii)}.

 Next by the chain rule and \er{uguyytfdddd} we obtain
\begin{equation}\label{uguyytfddddlhhjlk}
\nabla_{\vec x'}\psi'=\nabla_{\vec
x'}\psi=\left\{A^{-1}(t)\right\}^T\cdot\nabla_{\vec
x}\psi=A(t)\cdot\nabla_{\vec x}\psi,
\end{equation}
by \er{uguyytfddddgghjjg} and \er{vyguiuiujggghjjgredbstr} we obtain
\begin{equation}\label{uguyytfddddgghjjgghghhf}
curl_{\vec x'}\vec f'=curl_{\vec x'}\left(A(t)\cdot\vec
f\right)=A(t)\cdot curl_{\vec x}\vec f,
\end{equation}
and by the chain rule and \er{uguyytfddddgghjjghjjj} we have
\begin{equation}\label{uguyytfddddgghjjghjjjjhghghghjjff}
div_{\vec x'}T'=div_{\vec x'}\left(A(t)\cdot T\cdot
A^T(t)\right)=A(t)\cdot\left(div_{\vec x}T\right).
\end{equation}
Thus by \er{uguyytfddddgghjjghjjjjhghghghjjff} and
\er{NoIn5redbstruiytrrhjjgghjkgkgh} we have
\begin{equation}
\label{NoIn5redbstruiytrrhjjgghjkgkghfhfddgh} div_{\vec
x'}\left(d_{\vec x'}\vec v'+\left\{d_{\vec x'}\vec
v'\right\}^T\right)=A(t)\cdot\left\{div_{\vec x}\left(d_{\vec x}\vec
v+\left\{d_{\vec x}\vec v\right\}^T\right)\right\}.
\end{equation}
On the other hand by \er{uguyytfddddgghjjguyyjyhkkyh} and
\er{uguyytfddddlhhjlk} we have
\begin{equation}
\label{NoIn5redbstrgfhgfhfhgffgfghgjhjgjjj} \nabla_{\vec
x'}\left(div_{\vec x'}\vec v'\right)=A(t)\cdot\nabla_{\vec
x}\left(div_{\vec x}\vec v\right).
\end{equation}
Therefore, by \er{NoIn5redbstruiytrrhjjgghjkgkghfhfddgh} and
\er{NoIn5redbstrgfhgfhfhgffgfghgjhjgjjj}, using \er{apfrm12} we
deduce
\begin{equation}
\label{NoIn5redbstrgfhgfhfhgffgfghgjhjgjjjghjjjf} \Delta_{\vec
x'}\vec v'=A(t)\cdot\Delta_{\vec x}\vec v \quad\text{and}\quad
curl_{\vec x'}\left(curl_{\vec x'}\vec v'\right)=A(t)\cdot
curl_{\vec x}\left(curl_{\vec x}\vec v\right).
\end{equation}
Next by \er{NoIn5redbstr} we deduce
\begin{equation}\label{hjhkujgjgyjfhijkyuiu}
(\vec u'-\vec v')=A(t)\cdot(\vec u-\vec v).
\end{equation}
So by \er{uguyytfddddgghjjguyyjyhkkyhhjhjh}, \er{uguyytfddddlhhjlk},
\er{uguyytfddddgghjjgghghhf},
\er{uguyytfddddgghjjghjjjjhghghghjjff},
\er{NoIn5redbstruiytrrhjjgghjkgkghfhfddgh},
\er{NoIn5redbstrgfhgfhfhgffgfghgjhjgjjj},
\er{NoIn5redbstrgfhgfhfhgffgfghgjhjgjjjghjjjf} and
\er{hjhkujgjgyjfhijkyuiu} we deduce {\bf(ii)}.

Furthermore, by the chain rule for every scalar field
$\gamma:\R^3\times[t_0,+\infty)\to\R$ and for every vector field
$\vec \Gamma:\R^3\times[t_0,+\infty)\to\R^3$ we obtain
\begin{equation}\label{hjhkujgjgyjfhijk}
\frac{\partial \vec \gamma}{\partial t}=\frac{\partial \vec
\gamma}{\partial t'}+\left(A'(t)\cdot\vec x+\vec
w(t)\right)\cdot\nabla_{\vec x'}\vec \gamma
\end{equation}
and
\begin{equation}\label{hjhkujgjgyjf}
\frac{\partial \vec \Gamma}{\partial t}=\frac{\partial \vec
\Gamma}{\partial t'}+\left(d_{\vec x'}\vec
\Gamma\right)\cdot\left(A'(t)\cdot\vec x+\vec w(t)\right).
\end{equation}
Therefore, by \er{hjhkujgjgyjf} and \er{vyguiuiujggghjjgredbstr}
\begin{equation}
\label{vyguiuiujggghjjgggjredbstr} \frac{\partial \vec
\Gamma}{\partial t'}=\frac{\partial \vec \Gamma}{\partial
t}-\left(d_{\vec x}\vec \Gamma\right)\cdot \left(A^{-1}(t)\cdot
A'(t)\cdot\vec x+A^{-1}(t)\cdot \vec w(t)\right),
\end{equation}
and by \er{vyguiuiujggghjjgredbstr} \er{uguyytfddddlhhjlk} and
\er{hjhkujgjgyjfhijk}
\begin{equation}\label{hjhkujgjgyjfhijkjkhkggjggj}
\frac{\partial \vec \gamma}{\partial t'}+\left(A(t)\cdot\vec
\Gamma+A'(t)\cdot\vec x+\vec w(t)\right)\cdot\nabla_{\vec x'}\vec
\gamma=\frac{\partial \vec \gamma}{\partial
t}+\vec\Gamma\cdot\nabla_{\vec x}\gamma.
\end{equation}
In particular, by \er{uguyytfdddd}, \er{NoIn5redbstr} and
\er{hjhkujgjgyjfhijkjkhkggjggj} we have
\begin{equation}\label{hjhkujgjgyjfhijkjkhkggjggjugiu}
\frac{\partial \vec \psi}{\partial t'}+\vec v'\cdot\nabla_{\vec
x'}\vec \psi=\frac{\partial \vec \psi}{\partial t}+\vec
v\cdot\nabla_{\vec x}\psi
\end{equation}
and then since
\begin{equation}\label{hjhkujgjgyjfhijkjkhkggjggjugiujghj}
\frac{\partial \vec \psi}{\partial t}+div_{\vec x}\left\{\psi\vec
v\right\}=\frac{\partial \vec \psi}{\partial t}+\vec
v\cdot\nabla_{\vec x}\psi+\psi\left(div_{\vec x}\vec v\right),
\end{equation}
by \er{hjhkujgjgyjfhijkjkhkggjggjugiu}, \er{uguyytfdddd} and
\er{NoIn5redbstrgfhgfhfhgffgf} we infer {\bf (iv)}. On the other
hand,
%
%
%
%
%
%
by \er{vyguiuiujggghjjgredbstr}, \er{vyguiuiujggghjjgggjredbstr} and
\er{NoIn5redbstr}
for every vector field $\vec \Gamma:\R^3\times[t_0,+\infty)\to\R^3$
we get:
\begin{multline}\label{bgvfgfhhtgjtggguyfredbstr}
\frac{\partial \left(A(t)\cdot\vec \Gamma\right)}{\partial
t'}-curl_{\vec x'}\left(\vec v'\times\left(A(t)\cdot\vec
\Gamma\right)\right)+\left({div}_{\vec x'}\left(A(t)\cdot\vec
\Gamma\right)\right)\vec v'=\\
\left(A(t)\cdot\frac{\partial \vec \Gamma}{\partial
t}+A'(t)\cdot\vec \Gamma-A(t)\cdot \left(d_{\vec x}\vec
\Gamma\right)\cdot \left(A^{-1}(t)\cdot A'(t)\cdot\vec
x+A^{-1}(t)\cdot \vec w(t)\right)\right)\\-A(t)\cdot curl_{\vec
x}\left(\left(\vec v+A^{-1}(t)\cdot A'(t)\cdot\vec
x+A^{-1}(t)\cdot\vec w(t)\right)\times
\vec\Gamma\right)\\+\left({div}_{\vec x}\vec
\Gamma\right)\left(A(t)\cdot\vec v+A'(t)\cdot\vec x+\vec
w(t)\right)\\= A(t)\cdot\left(\frac{\partial \vec \Gamma}{\partial
t}- curl_{\vec x}\left(\vec v\times
\vec\Gamma\right)+\left({div}_{\vec x}\vec \Gamma\right)\vec
v\right)\\+A(t)\cdot \left(d_{\vec x}\left(A^{-1}(t)\cdot
A'(t)\cdot\vec x+A^{-1}(t)\cdot\vec w(t)\right)\right)\cdot\vec
\Gamma\\-A(t)\cdot \left(d_{\vec x}\vec \Gamma\right)\cdot
\left(A^{-1}(t)\cdot A'(t)\cdot\vec x+A^{-1}(t)\cdot \vec
w(t)\right)\\+A(t)\cdot\left(\left({div}_{\vec x}\vec
\Gamma\right)\left(A^{-1}(t)\cdot A'(t)\cdot\vec
x+A^{-1}(t)\cdot\vec w(t)\right)\right)\\-A(t)\cdot curl_{\vec
x}\left(\left(A^{-1}(t)\cdot A'(t)\cdot\vec x+A^{-1}(t)\cdot\vec
w(t)\right)\times \vec\Gamma\right).
\end{multline}
On the other hand,  by \er{apfrm6} we have,
\begin{multline}\label{bgvfgfhhtgjtgjuiyuijjkredbstr}
\left(d_{\vec x}\left(A^{-1}(t)\cdot A'(t)\cdot\vec
x+A^{-1}(t)\cdot\vec w(t)\right)\right)\cdot\vec
\Gamma\\-\left(d_{\vec x}\vec \Gamma\right)\cdot
\left(A^{-1}(t)\cdot A'(t)\cdot\vec x+A^{-1}(t)\cdot \vec
w(t)\right)\\+\left({div}_{\vec x}\vec
\Gamma\right)\left(A^{-1}(t)\cdot A'(t)\cdot\vec
x+A^{-1}(t)\cdot\vec w(t)\right)\\- curl_{\vec
x}\left(\left(A^{-1}(t)\cdot A'(t)\cdot\vec x+A^{-1}(t)\cdot\vec
w(t)\right)\times \vec\Gamma\right)\\=\left({div}_{\vec
x}\left(A^{-1}(t)\cdot A'(t)\cdot\vec x+A^{-1}(t)\cdot\vec
w(t)\right)\right)\vec \Gamma.
\end{multline}
Therefore, by \er{bgvfgfhhtgjtggguyfredbstr} and
\er{bgvfgfhhtgjtgjuiyuijjkredbstr} we deduce:
\begin{multline}\label{bgvfgfhhtgjtgredbstr}
\frac{\partial \left(A(t)\cdot\vec \Gamma\right)}{\partial
t'}-curl_{\vec x'}\left(\vec v'\times\left(A(t)\cdot\vec
\Gamma\right)\right)+\left({div}_{\vec x'}\left(A(t)\cdot\vec
\Gamma\right)\right)\vec v'=\\ A(t)\cdot\left(\frac{\partial \vec
\Gamma}{\partial t}- curl_{\vec x}\left(\vec v\times
\vec\Gamma\right)+\left({div}_{\vec x}\vec \Gamma\right)\vec
v\right)\\+A(t)\cdot\left(\left({div}_{\vec x}\left(A^{-1}(t)\cdot
A'(t)\cdot\vec x+A^{-1}(t)\cdot\vec w(t)\right)\right)\vec
\Gamma\right)
\\=A(t)\cdot\left(\frac{\partial \vec \Gamma}{\partial t}- curl_{\vec
x}\left(\vec v\times \vec\Gamma\right)+\left({div}_{\vec x}\vec
\Gamma\right)\vec v\right)+\left(tr\left(A^{-1}(t)\cdot
A'(t)\right)\right)A(t)\cdot\vec \Gamma,
\end{multline}
where $tr\left(A^{-1}(t)\cdot A'(t)\right)$ is the trace of the
matrix $A^{-1}(t)\cdot A'(t)$.
Therefore, by \er{bgvfgfhhtgjtgredbstr} and
\er{NoIn5redbstrgfhgfhfhuhhjhjj} for every vector field $\vec
\Gamma:\R^3\times[t_0,+\infty)\to\R^3$ we have:
\begin{multline}\label{bgvfgfhhtgjtggiuguiuiredbstr}
\frac{\partial \left(A(t)\cdot\vec \Gamma\right)}{\partial
t'}-curl_{\vec x'}\left(\vec v'\times\left(A(t)\cdot\vec
\Gamma\right)\right)+\left({div}_{\vec x'}\left(A(t)\cdot\vec
\Gamma\right)\right)\vec v'\\=A(t)\cdot\left(\frac{\partial \vec
\Gamma}{\partial t}- curl_{\vec x}\left(\vec v\times
\vec\Gamma\right)+\left({div}_{\vec x}\vec \Gamma\right)\vec
v\right).
\end{multline}
Thus, by \er{bgvfgfhhtgjtggiuguiuiredbstr} and
\er{uguyytfddddgghjjg} we infer
\begin{equation}\label{bgvfgfhhtgjtggiuguiuiredbstrkgghkj}
\frac{\partial \vec f'}{\partial t'}-curl_{\vec x'}\left(\vec
v'\times\vec f'\right)+\left({div}_{\vec x'}\vec f'\right)\vec
v'=A(t)\cdot\left(\frac{\partial \vec f}{\partial t}- curl_{\vec
x}\left(\vec v\times \vec f\right)+\left({div}_{\vec x}\vec
f\right)\vec v\right).
\end{equation}
Finally, by \er{apfrm9}, \er{apfrm7} and \er{apfrm6} we deduce
\begin{multline}\label{vhfffngghhjghhgjlkhjhk}
\frac{\partial\vec f}{\partial t}-\vec v\times curl_{\vec x}\vec
f+\nabla_{\vec x}\left(\vec f\cdot\vec v\right)=\nabla_{\vec
x}\left(\vec f\cdot\vec v\right)+\frac{\partial\vec f}{\partial
t}+d_{\vec x}\vec f\cdot\vec v-\left\{d_{\vec x}\vec
f\right\}^T\cdot\vec v\\=\frac{\partial\vec f}{\partial t}+d_{\vec
x}\vec f\cdot\vec v+ \left\{d_{\vec x}\vec v\right\}^T\cdot\vec
f=\frac{\partial\vec f}{\partial t}+d_{\vec x}\vec f\cdot\vec
v-d_{\vec x}\vec v\cdot\vec f+\left(d_{\vec x}\vec v+\left\{d_{\vec
x}\vec v\right\}^T\right)\cdot\vec f\\= \left(\frac{\partial\vec
f}{\partial t}-curl_{\vec x}\left(\vec v\times\vec
f\right)+\left(div_{\vec x}\vec f\right)\vec
v\right)-\left(div_{\vec x}\vec v\right)\vec f+ \left(d_{\vec x}\vec
v+\left\{d_{\vec x}\vec v\right\}^T\right)\cdot\vec f.
\end{multline}
So we get \er{vhfffngghhjghhgjlkhjhkPPP}. Moreover, by
\er{uguyytfddddgghjjg}, \er{NoIn5redbstrgfhgfhfhgffgf},
\er{NoIn5redbstruiytrrhjjgghjkgkgh}, \er{vhfffngghhjghhgjlkhjhk} and
\er{bgvfgfhhtgjtggiuguiuiredbstrkgghkj} we infer
\begin{equation}\label{vhfffngghhjghhgjlkhjhkfghhffhhhhg}
\frac{\partial\vec f'}{\partial t'}-\vec v'\times curl_{\vec x'}\vec
f'+\nabla_{\vec x'}\left(\vec f'\cdot\vec
v'\right)=A(t)\cdot\left(\frac{\partial\vec f}{\partial t}-\vec
v\times curl_{\vec x}\vec f+\nabla_{\vec x}\left(\vec f\cdot\vec
v\right)\right).
\end{equation}
So by \er{bgvfgfhhtgjtggiuguiuiredbstrkgghkj} and
\er{vhfffngghhjghhgjlkhjhkfghhffhhhhg} we finally obtain {\bf (v)}.
\end{proof}

\begin{proposition}\label{yghgjtgyrtrgjgjgjgjgjgjgjgjgjgjgjgjgj33}
Consider some fixed inertial or non-inertial cartesian coordinate
system, denoted by the sign $(\{0\})$, and consider a vector field
$\vec d:=\vec d(\vec x,t):\R^3\times[t_0,+\infty)\to\R^3$ associated
\underline{only} with the chosen coordinate system $(\{0\})$. Then
there exist a unique \underline{speed-like} vector field $\vec
v_{\vec d}:\R^3\times[t_0,+\infty)\to\R^{3}$ (associated with any
inertial or non-inertial cartesian coordinate system) such that in
the \underline{particular} system $(\{0\})$ we have:
\begin{equation}\label{vfyutuyfffhhgfhgfh33iy33}
\vec v_{\vec d}(\vec x,t)=\vec d(\vec x,t)\quad\quad\forall(\vec
x,t) \in\R^3\times[t_0,+\infty)\,.
\end{equation}
Moreover, there exist unique a \underline{proper} vector field $\vec
h_{\vec d}:\R^3\times[t_0,+\infty)\to\R^3$ (associated with any
inertial or non-inertial cartesian coordinate system) such that in
the \underline{particular} system $(\{0\})$ we have:
\begin{equation}\label{vfyutuyfffhhgfhgfh33iyjkkj33}
\vec h_{\vec d}(\vec x,t)=\vec d(\vec x,t)\quad\quad\forall(\vec
x,t) \in\R^3\times[t_0,+\infty)\,.
\end{equation}
\end{proposition}
\begin{proof}
Since uniqueness is trivial, we concentrate ourselves on existence.
Similarly to \er{noninchgravortbstr}, given an arbitrary inertial or
non-inertial cartesian coordinate system, denoted by the sign
$(\{1\})$, the change of coordinate system $(\{0\})$ to the system
$(\{1\})$ can be written as:
\begin{equation}\label{noninchgravortbstrjkjkbbq133}
\begin{cases}
\vec x^{(1)}=B_1(t)\cdot\vec x+\vec z_1(t),\\
t^{(1)}=t,
\end{cases}
\end{equation}
where $B_1(t)\in SO(3)$ is a rotation, i.e. $B_1(t)\in \R^{3\times
3}$, $det\, B_1(t)>0$ and $B_1(t)\cdot B^T_1(t)=I$. Similarly, given
another arbitrary inertial or non-inertial cartesian coordinate
system, denoted by the sign $(\{2\})$, the change of coordinate
system $(\{0\})$ to the system $(\{2\})$ can be written as:
\begin{equation}\label{noninchgravortbstrjkjkbbq233}
\begin{cases}
\vec x^{(2)}=B_2(t)\cdot\vec x+\vec z_2(t),\\
t^{(2)}=t,
\end{cases}
\end{equation}
where $B_2(t)\in SO(3)$ is another rotation. On the other hand, the
change coordinate system $(\{1\})$ to the system $(\{2\})$ can be
written as:
\begin{equation}\label{noninchgravortbstrjkjkjggg33}
\begin{cases}
\vec x^{(2)}=A(t^{(1)})\cdot\vec x^{(1)}+\vec z(t^{(1)}),\\
t^{(2)}=t^{(1)},
\end{cases}
\end{equation}
where $A(t)\in SO(3)$ is a rotation. In particular, inserting
\er{noninchgravortbstrjkjkbbq133} into
\er{noninchgravortbstrjkjkjggg33} gives
\begin{equation}\label{noninchgravortbstrjkjkjggg33hkjh33}
\vec x^{(2)}=A(t)\cdot\left(B_1(t)\cdot\vec x+\vec
z_1(t)\right)+\vec z(t)=\left(A(t)\cdot B_1(t)\right)\cdot\vec
x+\left(A(t)\cdot\vec z_1(t)+\vec z(t)\right).
\end{equation}
Thus, comparing \eqref{noninchgravortbstrjkjkjggg33hkjh33} with
\er{noninchgravortbstrjkjkbbq233} we easily deduce:
\begin{equation}\label{noninchgravortbstrjkjkjggg33hkjh33jklhhk33}
B_2(t)=\left(A(t)\cdot
B_1(t)\right)\quad\quad\text{and}\quad\quad\vec
z_2(t)=\left(A(t)\cdot\vec z_1(t)+\vec
z(t)\right)\quad\quad\quad\forall t.
\end{equation}

Next, for the system $(\{1\})$ define
\begin{equation}
\label{NoIn5redbstrhkhhk133}\vec v^{(1)}_{\vec d}(\vec
x^{(1)},t^{(1)})=B_1(t)\cdot \vec d(\vec x,t)+
\frac{d B_1}{dt}(t)\cdot\vec x+
\frac{d\vec z_1}{dt}(t),
\end{equation}
and
\begin{equation}\label{uguyytfddddgghjjgjhjhjh133}
\vec h^{(1)}_{\vec d}(\vec x^{(1)},t^{(1)})=B_1(t)\cdot\vec d(\vec
x,t).
\end{equation}
Correspondingly, for the system $(\{2\})$ we have
\begin{equation}
\label{NoIn5redbstrhkhhk233}\vec v^{(2)}_{\vec d}(\vec
x^{(2)},t^{(2)})=B_2(t)\cdot \vec d(\vec x,t)+
\frac{d B_2}{dt}(t)\cdot\vec x+
\frac{d\vec z_2}{dt}(t),
\end{equation}
and
\begin{equation}\label{uguyytfddddgghjjgjhjhjh233}
\vec h^{(2)}_{\vec d}(\vec x^{(2)},t^{(2)})=B_2(t)\cdot\vec d(\vec
x,t).
\end{equation}
Thus, in the \underline{particular} system $(\{0\})$ we have:
\begin{equation}\label{vfyutuyfffhhgfhgfh33iy33hkjhjh33}
\vec v_{\vec d}(\vec x,t)=\vec h_{\vec d}(\vec x,t)=\vec d(\vec
x,t)\quad\quad\forall(\vec x,t) \in\R^3\times[t_0,+\infty)\,.
\end{equation}
Thus in order to complete the proof we need to show that the vector
field, defined by \er{NoIn5redbstrhkhhk133} and
\er{NoIn5redbstrhkhhk233} is a speed-like vector field and the
vector field, defined by \er{uguyytfddddgghjjgjhjhjh133} and
\er{uguyytfddddgghjjgjhjhjh233} is a proper vector field. Indeed, by
\er{uguyytfddddgghjjgjhjhjh133} and \er{uguyytfddddgghjjgjhjhjh233}
together with \er{noninchgravortbstrjkjkjggg33hkjh33jklhhk33} and
\er{noninchgravortbstrjkjkbbq133} we easily deduce:
\begin{equation}\label{uguyytfddddgghjjgjhjhjh233hkjhjh33}
\vec h^{(2)}_{\vec d}(\vec x^{(2)},t^{(2)})=\left(A(t)\cdot
B_1(t)\right)\cdot\vec d(\vec x,t)=A(t)\cdot\left(B_1(t)\cdot\vec
d(\vec x,t)\right)=A(t^{(1)})\cdot\vec h^{(1)}_{\vec d}(\vec
x^{(1)},t^{(1)}),
\end{equation}
and so, comparing \er{uguyytfddddgghjjgjhjhjh233hkjhjh33} with
\er{uguyytfddddgghjjg} in Definition \eqref{bggghghgj}, we deduce
that $\vec h_{\vec d}$ is a \underline{proper} vector field.
Similarly,
 by \er{NoIn5redbstrhkhhk133} and
\er{NoIn5redbstrhkhhk233} together with
\er{noninchgravortbstrjkjkjggg33hkjh33jklhhk33} and
\er{noninchgravortbstrjkjkbbq133} we deduce:
\begin{multline}
\label{NoIn5redbstrhkhhk233jhgjgj33}\vec v^{(2)}_{\vec d}(\vec
x^{(2)},t^{(2)})=B_2(t)\cdot \vec d(\vec x,t)+
\frac{d B_2}{dt}(t)\cdot\vec x+
\frac{d\vec z_2}{dt}(t)=
\\
\left(A(t)\cdot B_1(t)\right)\cdot \vec d(\vec x,t)+
\left(\frac{d}{dt}\left\{A(t)\cdot B_1(t)\right\}\right)\cdot\vec x+
\frac{d}{dt}\left\{A(t)\cdot\vec z_1(t)+\vec z(t)\right\}=
\\ A(t)\cdot
\left(B_1(t)\cdot \vec d(\vec x,t)\right)+ \left(A(t)\cdot
\frac{dB_1}{dt}(t)+\frac{dA}{dt}(t)\cdot B_1(t)\right)\cdot\vec
x+\left\{A(t)\cdot\frac{d\vec z_1}{dt}(t)+\frac{dA}{dt}(t)\cdot\vec
z_1(t)+\frac{d\vec z}{dt}(t)\right\}\\=
 A(t)\cdot
\left(B_1(t)\cdot \vec d(\vec x,t)+\frac{dB_1}{dt}(t)\cdot\vec
x+\frac{d\vec z_1}{dt}(t)\right)+\frac{dA}{dt}(t)\cdot\left(
B_1(t)\cdot\vec x+\vec z_1(t)\right)+\frac{d\vec z}{dt}(t)\\= \vec
A(t^{(1)})\cdot \vec v^{(1)}_{\vec d}(\vec
x^{(1)},t^{(1)})+\frac{dA}{dt^{(1)}}(t^{(1)})\cdot\vec
x^{(1)}+\frac{d\vec z}{dt^{(1)}}(t^{(1)}),
\end{multline}
and so, comparing \er{NoIn5redbstrhkhhk233jhgjgj33} with
\er{NoIn5redbstr} in Definition \ref{bggghghgj}, we deduce that
$\vec v_{\vec d}$ is a \underline{speed-like} vector field. This
completes the proof.
\end{proof}

\begin{definition}\label{jhjkhhjhjhgjg}
Let $\vec u:=\vec u(\vec x,t)$ be a speed-like vector field defined
for $(\vec x,t)\in\Omega$ with domain of definition
$\Omega\subset\mathbb{R}^3\times\mathbb{R}$ and let $(*)$ be some
inertial or non-inertial cartesian coordinate system. Then $\vec u$
is called \underline{generally trivial speed-like field} if there
exists another cartesian coordinate system $(**)$ such that under
the change of coordinate system $(*)$ to another cartesian
coordinate system $(**)$ given by:
\begin{equation}\label{noninchgravortbstrjgghguittu2hhhhh}
\begin{cases}
\vec x'=A(t)\cdot\vec x+\vec z(t),\\
t'=t,
\end{cases}
\end{equation}
where $A(t)\in SO(3)$ is a rotation, we have
\begin{equation}
\label{jljjjjkhhh}A(t)\cdot \vec u(\vec x,t)+
\frac{d A}{dt}(t)\cdot\vec x+
\frac{d\vec z}{dt}(t)=\vec u'(\vec x',t')=0.
\end{equation}
\end{definition}
\begin{proposition}
\label{jhjkhhjhjhgjgjjkjk} Let $\vec u:=\vec u(\vec x,t)$ be a
smooth speed-like vector field defined for $(\vec x,t)\in\Omega$
with connected domain $\Omega\subset\mathbb{R}^3\times\mathbb{R}$,
such that for every instant of time $\tau$ the domain
$\Omega_\tau=\left\{\vec x\in \mathbb{R}^3:\,(\vec
x,\tau)\in\Omega\right\}$ is also connected. Then $\vec u$ is a
\underline{generally trivial speed-like field} if and only if we
have
\begin{equation}
\label{jljjjjkhhhjhjhjh} d_{\vec x}\vec u+\left\{d_{\vec x}\vec
u\right\}^T=0\quad\quad\forall (\vec x,t)\in\Omega.
\end{equation}
\end{proposition}
\begin{proof}
Assume that \er{jljjjjkhhhjhjhjh} is satisfied and denote
$(u_1,u_2,u_3)=\vec u\in\mathbb{R}^3$ and $(x_1,x_2,x_3)=\vec
x\in\mathbb{R}^3$. Then \er{jljjjjkhhhjhjhjh} is equivalent to
\begin{equation}
\label{jljjjjkhhhjhjhjhjgjg} \frac{\partial u_m}{\partial
x_n}+\frac{\partial u_n}{\partial x_m}=0\quad\quad\forall (\vec
x,t)\in\Omega,\;\;\forall m,n=1,2,3.
\end{equation}
One more time differentiating of \er{jljjjjkhhhjhjhjhjgjg} gives
\begin{equation}
\label{jljjjjkhhhjhjhjhjgjgbngjgjhjhhj} \frac{\partial^2
u_m}{\partial x_n\partial x_k}+\frac{\partial^2 u_n}{\partial
x_m\partial x_k}=0\quad\quad\forall (\vec x,t)\in\Omega,\;\;\forall
m,n,k=1,2,3.
\end{equation}
Therefore,
\begin{multline}
\label{jljjjjkhhhjhjhjhjgjgbngjgjhjhhjhhh} 2\frac{\partial^2
u_n}{\partial x_m\partial x_k}=\left(\frac{\partial^2 u_m}{\partial
x_n\partial x_k}+\frac{\partial^2 u_n}{\partial x_m\partial
x_k}\right)-\left(\frac{\partial^2 u_m}{\partial x_k\partial
x_n}+\frac{\partial^2 u_k}{\partial x_m\partial
x_n}\right)\\+\left(\frac{\partial^2 u_k}{\partial x_n\partial
x_m}+\frac{\partial^2 u_n}{\partial x_k\partial
x_m}\right)=0\quad\quad\forall (\vec x,t)\in\Omega,\;\;\forall
m,n,k=1,2,3.
\end{multline}
So,
\begin{equation}
\label{jljjjjkhhhjhjhjhjgjgbngjgjhjhhjkjkhhhjhj}
\frac{\partial}{\partial x_k}\left(\frac{\partial u_n}{\partial
x_m}\right)=\frac{\partial^2 u_n}{\partial x_m\partial
x_k}=0\quad\quad\forall (\vec x,t)\in\Omega,\;\;\forall m,n,k=1,2,3.
\end{equation}
Thus since, for every instant of time $t$ the domain
$\Omega_t\subset \mathbb{R}^3$ is connected, by
\er{jljjjjkhhhjhjhjhjgjgbngjgjhjhhjkjkhhhjhj} we deduce that
$d_{\vec x}\vec u$ is independent of the spatial variable $\vec x$
and can depend only on the time variable $t$. So by
\er{jljjjjkhhhjhjhjhjgjgbngjgjhjhhjkjkhhhjhj} and
\er{jljjjjkhhhjhjhjhjgjg} together we deduce that there exists an
antisymmetric matrix $B(t)\in\mathbb{R}^{3\times 3}$, depending on
the time variable $t$ only, such that
\begin{equation}
\label{jljjjjkhhhjhjhjhjgjgbngjgjhjhhjkjkhhhjhjggjhjhhjgjg} d_{\vec
x}\vec u=B(t)\quad\forall (\vec
x,t)\in\Omega\quad\quad\text{where}\quad B^T(t)=-B(t).
\end{equation}
Thus there exists a regular mapping $\vec
d(t):\mathbb{R}\to\mathbb{R}^3$ such that
\begin{equation}
\label{jljjjjkhhhjhjhjhjgjgbngjgjhjhhjkjkhhhjhjggjhjhhjgjghhjhjhjjh}
\vec u(\vec x,t)=B(t)\cdot \vec x+\vec d(t)\quad\forall (\vec
x,t)\in\Omega\quad\quad\text{where}\quad B^T(t)=-B(t).
\end{equation}
Next let
%
%
%
%
%
%
%
%
$\vec h(t)\in\mathbb{R}^3$ be a solution of
\begin{equation}
\label{jljjjjkhhhjhjhjhjgjgbngjgjhjhhjkjkhhhjhjggjhjhhjgjghhjhjhjjhgjgggjhhhjkhhkhhhkh}
\begin{cases}
\frac{d\vec h}{dt}(t)=B(t)\cdot\vec h(t)+\vec d(t)\quad\quad\forall
t
\\
\vec h(0)=0,
\end{cases}
\end{equation}
and $R(t)\in \mathbb{R}^{3\times 3}$ be a matrix-valued solution of
\begin{equation}
\label{jljjjjkhhhjhjhjhjgjgbngjgjhjhhjkjkhhhjhjggjhjhhjgjghhjhjhjjhgjgggjhhhjkhhkhhhkhjkjkhkhk}
\begin{cases}
\frac{d R}{dt}(t)=B(t)\cdot R(t)\quad\quad\forall t
\\
R(0)=I,
\end{cases}
\end{equation}
where $I$ is the identity matrix. In particular,
\begin{equation}
\label{jljjjjkhhhjhjhjhjgjgbngjgjhjhhjkjkhhhjhjggjhjhhjgjghhjhjhjjhgjgggjhhhjkhhkhhhkhjkjkhkhkhjjgjhhj}
R^T(t)\cdot\frac{d R}{dt}(t)=R^T(t)\cdot B(t)\cdot
R(t)\quad\quad\forall t.
\end{equation}
On the other hand, by
\er{jljjjjkhhhjhjhjhjgjgbngjgjhjhhjkjkhhhjhjggjhjhhjgjghhjhjhjjhgjgggjhhhjkhhkhhhkhjkjkhkhk}
and
\er{jljjjjkhhhjhjhjhjgjgbngjgjhjhhjkjkhhhjhjggjhjhhjgjghhjhjhjjh} we
have,
\begin{equation}
\label{jljjjjkhhhjhjhjhjgjgbngjgjhjhhjkjkhhhjhjggjhjhhjgjghhjhjhjjhgjgggjhhhjkhhkhhhkhjkjkhkhkjgkuuijhhkoilj}
\frac{d R^T}{dt}(t)=R^T(t)\cdot B^T(t)=-R^T(t)\cdot
B(t)\quad\quad\forall t,
\end{equation}
and thus
\begin{equation}
\label{jljjjjkhhhjhjhjhjgjgbngjgjhjhhjkjkhhhjhjggjhjhhjgjghhjhjhjjhgjgggjhhhjkhhkhhhkhjkjkhkhkjgkuuijhhkoiljhhhhjjhjjk}
\frac{d R^T}{dt}(t)\cdot R(t)=-R^T(t)\cdot B(t)\cdot
R(t)\quad\quad\forall t.
\end{equation}
Thus, by
\er{jljjjjkhhhjhjhjhjgjgbngjgjhjhhjkjkhhhjhjggjhjhhjgjghhjhjhjjhgjgggjhhhjkhhkhhhkhjkjkhkhkhjjgjhhj}
and
\er{jljjjjkhhhjhjhjhjgjgbngjgjhjhhjkjkhhhjhjggjhjhhjgjghhjhjhjjhgjgggjhhhjkhhkhhhkhjkjkhkhkjgkuuijhhkoiljhhhhjjhjjk}
together we obtain
\begin{equation}
\label{jljjjjkhhhjhjhjhjgjgbngjgjhjhhjkjkhhhjhjggjhjhhjgjghhjhjhjjhgjgggjhhhjkhhkhhhkhjkjkhkhkjgkuuijhhkoiljhhhhjjhjjkhkhkhkh}
\frac{d}{dt}\left(R^T(t)\cdot R(t)\right)=R^T(t)\cdot\frac{d
R}{dt}(t)+\frac{d R^T}{dt}(t)\cdot R(t)=0\quad\quad\forall t.
\end{equation}
Therefore, by
\er{jljjjjkhhhjhjhjhjgjgbngjgjhjhhjkjkhhhjhjggjhjhhjgjghhjhjhjjhgjgggjhhhjkhhkhhhkhjkjkhkhkjgkuuijhhkoiljhhhhjjhjjkhkhkhkh}
and
\er{jljjjjkhhhjhjhjhjgjgbngjgjhjhhjkjkhhhjhjggjhjhhjgjghhjhjhjjhgjgggjhhhjkhhkhhhkhjkjkhkhk}
we obtain
\begin{equation}
\label{jljjjjkhhhjhjhjhjgjgbngjgjhjhhjkjkhhhjhjggjhjhhjgjghhjhjhjjhgjgggjhhhjkhhkhhhkhjkjkhkhkjgkuuijhhkoiljhhhhjjhjjkhkhkhkhjkhhjgjhhhhj}
R^T(t)\cdot R(t)=R^T(0)\cdot R(0)=I^T\cdot I=I\quad\quad\forall t.
\end{equation}
In particular, we have $\left(det\,R(t)\right)^2=1$ and thus, since
$det\,R(0)=1$, we deduce $det\,R(t)=1$. So we obtain
\begin{equation}
\label{jljjjjkhhhjhjhjhjgjgbngjgjhjhhjkjkhhhjhjggjhjhhjgjghhjhjhjjhgjgggjhhhjkhhkhhhkhjkjkhkhkjgkuuijhhkoiljhhhhjjhjjkhkhkhkhjkhhjgjhhhhjgjkllukbbn}
R(t)\in SO(3)\quad\quad\forall t.
\end{equation}
Next consider $A(t)\in SO(3)$ and $\vec z(t)\in\mathbb{R}^3$ defined
by
\begin{equation}
\label{jljjjjkhhhjhjhjhjgjgbngjgjhjhhjkjkhhhjhjggjhjhhjgjghhjhjhjjhgjgggjhhhjkhhkhhhkhjkjkhkhkjgkuuijhhkoiljhhhhjjhjjkhkhkhkhjkhhjgjhhhhjgjkllukbbnjhhjkjk}
A(t)=R^T(t)=R^{-1}(t)\quad\text{and}\quad \vec z(t)=-R^T(t)\cdot\vec
h(t)\quad\quad\forall t,
\end{equation}
and consider the change of cartesian coordinate system given by:
\begin{equation}
\label{jljjjjkhhhjhjhjhjgjgbngjgjhjhhjkjkhhhjhjggjhjhhjgjghhjhjhjjhhkhhhhhiukk}
\begin{cases}\vec x'=A(t)\cdot\vec x+\vec z(t)=R^T(t)\cdot\vec
x-R^T(t)\cdot\vec h(t)
\\
t'=t.
\end{cases}
\end{equation}
Then since the vector field  $\vec u$ is a speed-like vector field,
under the change of coordinate system, given by
\er{jljjjjkhhhjhjhjhjgjgbngjgjhjhhjkjkhhhjhjggjhjhhjgjghhjhjhjjhhkhhhhhiukk},
$\vec u$ transforms as
\begin{equation}
\label{jljjjjkhhhgyggtt}\vec u'(\vec x',t')=A(t)\cdot \vec u(\vec
x,t)+
\frac{d A}{dt}(t)\cdot\vec x+
\frac{d\vec z}{dt}(t).
\end{equation}
Therefore, inserting
\er{jljjjjkhhhjhjhjhjgjgbngjgjhjhhjkjkhhhjhjggjhjhhjgjghhjhjhjjh}
into \er{jljjjjkhhhgyggtt} gives
\begin{equation}
\label{jljjjjkhhhjhjhjhjgjgbngjgjhjhhjkjkhhhjhjggjhjhhjgjghhjhjhjjhuyuyyy11}
\vec u'(\vec x',t')=A(t)\cdot\left(B(t)\cdot \vec x+\vec
d(t)\right)+ \frac{d A}{dt}(t)\cdot\vec x+ \frac{d\vec z}{dt}(t).
\end{equation}
Then inserting
\er{jljjjjkhhhjhjhjhjgjgbngjgjhjhhjkjkhhhjhjggjhjhhjgjghhjhjhjjhgjgggjhhhjkhhkhhhkh},
\er{jljjjjkhhhjhjhjhjgjgbngjgjhjhhjkjkhhhjhjggjhjhhjgjghhjhjhjjhgjgggjhhhjkhhkhhhkhjkjkhkhk}
and
\er{jljjjjkhhhjhjhjhjgjgbngjgjhjhhjkjkhhhjhjggjhjhhjgjghhjhjhjjhgjgggjhhhjkhhkhhhkhjkjkhkhkjgkuuijhhkoiljhhhhjjhjjkhkhkhkhjkhhjgjhhhhjgjkllukbbnjhhjkjk}
into
\er{jljjjjkhhhjhjhjhjgjgbngjgjhjhhjkjkhhhjhjggjhjhhjgjghhjhjhjjhuyuyyy11}
together with
\er{jljjjjkhhhjhjhjhjgjgbngjgjhjhhjkjkhhhjhjggjhjhhjgjghhjhjhjjhgjgggjhhhjkhhkhhhkhjkjkhkhkjgkuuijhhkoiljhhhhjjhjjkhkhkhkh}
gives
\begin{multline}
\label{jljjjjkhhhjhjhjhjgjgbngjgjhjhhjkjkhhhjhjggjhjhhjgjghhjhjhjjhuyuyyy}
\vec u'(\vec x',t')=R^T(t)\cdot\left(\frac{d R}{dt}(t)\cdot
R^T(t)\cdot \vec x+\vec d(t)\right)+ \frac{d R^T}{dt}(t)\cdot\vec x
-\frac{d}{dt}\left(R^T(t)\cdot\vec h(t)\right)=\\
R^T(t)\cdot\left(\frac{d R}{dt}(t)\cdot R^T(t)+R(t)\cdot \frac{d
R^T}{dt}(t)\right)\cdot \vec x +R^T(t)\cdot\vec
d(t)-\frac{dR^T}{dt}(t)\cdot\vec h(t)-R^T(t)\cdot\frac{d\vec h}{dt}(t)\\
=0 +R^T(t)\cdot\vec d(t)-\frac{dR^T}{dt}(t)\cdot\vec
h(t)-R^T(t)\cdot\left(B(t)\cdot\vec h(t)+\vec
d(t)\right)\\=R^T(t)\cdot\vec d(t)-\frac{dR^T}{dt}(t)\cdot\vec
h(t)-R^T(t)\cdot\left(\frac{dR}{dt}(t)\cdot R^T(t)\cdot\vec
h(t)+\vec d(t)\right)\\=0- R^T(t)\cdot\left(R(t)\cdot \frac{d
R^T}{dt}(t)+\frac{d R}{dt}(t)\cdot R^T(t)\right)\cdot \vec h(t)=0.
\end{multline}
So, starting from equality \er{jljjjjkhhhjhjhjh} we finally deduce
in
\er{jljjjjkhhhjhjhjhjgjgbngjgjhjhhjkjkhhhjhjggjhjhhjgjghhjhjhjjhuyuyyy}
that
\begin{equation}
\label{jljjjjkhhhjhjhjhjgjgbngjgjhjhhjkjkhhhjhjggjhjhhjgjghhjhjhjjhuyuyyyjhjhhj}
\vec u'(\vec x',t')=0\quad\quad\forall\,(\vec x',t')
\end{equation}
in some cartesian coordinate system and thus by the definition $\vec
u$ is \underline{a generally trivial speed-like field}.

Conversely, assume that $\vec u$ is \underline{a generally trivial
speed-like field} and $(*)$ is some inertial or non-inertial
cartesian coordinate system. Then by the definition
\ref{jhjkhhjhjhgjg} there exists another cartesian coordinate system
$(**)$ such that under the change of coordinate system $(*)$ to
another cartesian coordinate system $(**)$ given by
\er{noninchgravortbstrjgghguittu2hhhhh} we have \er{jljjjjkhhh}. In
particular, in system $(**)$ we have
\begin{equation}
\label{jljjjjkhhhjhjhjhhkhh} d_{\vec x'}\vec u'+\left\{d_{\vec
x'}\vec u'\right\}^T=0.
\end{equation}
Thus, since by Proposition \ref{yghgjtgyrtrt}, {\bf(iii)}, the
quantity $\left(d_{\vec x}\vec u+\left\{d_{\vec x}\vec
u\right\}^T\right)$ is a proper matrix field, by
\er{jljjjjkhhhjhjhjhhkhh} we infer that in system $(*)$ we also have
\begin{equation*}
d_{\vec x}\vec u+\left\{d_{\vec x}\vec u\right\}^T=0,
\end{equation*}
i.e. we obtain \er{jljjjjkhhhjhjhjh}.
\end{proof}

\begin{proposition}
\label{jhjkhhjhjhgjgjjkjkjhhjh} Let $\vec v:=\vec v(\vec
x,t):\mathbb{R}^3\times[t_0,+\infty)\to\mathbb{R}^3$ be a speed-like
vector field, such that there exists some cartesian coordinate
system ${(*)}$, where we have:
\begin{equation}
\label{jljjjjkhhhjhjhjhljjjljljouukhh} \lim\limits_{|\vec
x|\to+\infty}\vec v(\vec x,t)=0\quad\quad\forall t\quad\quad\text{(
in the given particular coordinate system ${(*)}$)}\,.
\end{equation}
Then there exist a uniquely defined \underline{generally trivial
speed-like} vector field $\,\vec k:=\vec k(\vec
x,t):\mathbb{R}^3\times[t_0,+\infty)\to\mathbb{R}^3$ and a uniquely
defined \underline{proper} vector field $\,\vec h:=\vec h(\vec
x,t):\mathbb{R}^3\times[t_0,+\infty)\to\mathbb{R}^3$, so that in
every cartesian coordinate system  we have
\begin{equation}
\label{jljjjjkhhhjhjhjhljjjljljouukhhjhjg} \lim\limits_{|\vec
x|\to+\infty}\vec h(\vec x,t)=0\quad\quad\forall t\,,
\end{equation}
and in every cartesian coordinate system we can decompose vector
field $\vec v$ as:
\begin{equation}
\label{jljjjjkhhhjhjhjhljjjljljkjk} \vec v(\vec x,t)=\vec h(\vec
x,t)+\vec k(\vec x,t)\quad\quad\quad\quad\forall (\vec
x,t)\in\mathbb{R}^3\times[t_0,+\infty).
\end{equation}
\end{proposition}
\begin{proof}
Let ${(*)}$ be a cartesian coordinate system, where
\er{jljjjjkhhhjhjhjhljjjljljouukhh} satisfied. Then, by Proposition
\ref{yghgjtgyrtrgjgjgjgjgjgjgjgjgjgjgjgjgj33} there exist a unique
\underline{speed-like} vector field $\vec
k:\R^3\times[t_0,+\infty)\to\R^{3}$, such that in the given
\underline{particular} system $(*)$ we have:
\begin{equation}\label{vfyutuyfffhhgfhgfh33iy33yhijhj}
\vec k(\vec x,t)=0\quad\quad\forall(\vec x,t)
\in\R^3\times[t_0,+\infty)\,.
\end{equation}
Obviously, by \er{vfyutuyfffhhgfhgfh33iy33yhijhj} we obtain that
$\vec k$ is \underline{generally trivial} speed-like vector field.
Furthermore, by the same Proposition
\ref{yghgjtgyrtrgjgjgjgjgjgjgjgjgjgjgjgjgj33} there exist a unique
\underline{proper} vector field $\vec
h:\R^3\times[t_0,+\infty)\to\R^3$, such that in the given
\underline{particular} system $(*)$ we have:
\begin{equation}\label{vfyutuyfffhhgfhgfh33iyjkkj33hjjggjhjj}
\vec h(\vec x,t)=\vec v(\vec x,t)\quad\quad\forall(\vec x,t)
\in\R^3\times[t_0,+\infty)\,.
\end{equation}
In particular, by \er{vfyutuyfffhhgfhgfh33iyjkkj33hjjggjhjj} and
\er{jljjjjkhhhjhjhjhljjjljljouukhh} in the given
\underline{particular} system $(*)$ we have:
\begin{equation}
\label{jljjjjkhhhjhjhjhljjjljljouukhhjggjkkh} \lim\limits_{|\vec
x|\to+\infty}\vec h(\vec x,t)=0\quad\quad\forall t\,.
\end{equation}
Moreover, by together \er{vfyutuyfffhhgfhgfh33iy33yhijhj} and
\er{vfyutuyfffhhgfhgfh33iyjkkj33hjjggjhjj} in the given
\underline{particular} system $(*)$ we also have:
\begin{equation}\label{vfyutuyfffhhgfhgfh33iyjkkj33hjjggjhjjkhhjkk}
\vec h(\vec x,t)=\vec v(\vec x,t)-\vec k(\vec
x,t)\quad\quad\forall(\vec x,t) \in\R^3\times[t_0,+\infty)\,.
\end{equation}
However, since $\vec v$ and $\vec k$ are both \underline{speed-like}
vector fields, by  {\bf(i)} in Proposition \ref{yghgjtgyrtrt} we
deduce that $(\vec v-\vec k)$ is a \underline{proper} vector field.
Thus, since $(\vec v-\vec k)$ and $\vec h$ are both
\underline{proper} vector fields, we deduce that
\er{jljjjjkhhhjhjhjhljjjljljouukhhjggjkkh} and
\er{vfyutuyfffhhgfhgfh33iyjkkj33hjjggjhjjkhhjkk} hold in
\underline{every} cartesian coordinate system. So we deduce,
\er{jljjjjkhhhjhjhjhljjjljljouukhhjhjg} and
\er{jljjjjkhhhjhjhjhljjjljljkjk}.

Finally, assume that an arbitrary \underline{generally trivial}
speed-like vector field $\,\vec k:=\vec k(\vec
x,t):\mathbb{R}^3\times[t_0,+\infty)\to\mathbb{R}^3$ and an
arbitrary \underline{proper} vector field $\,\vec h:=\vec h(\vec
x,t):\mathbb{R}^3\times[t_0,+\infty)\to\mathbb{R}^3$ satisfy
\er{jljjjjkhhhjhjhjhljjjljljouukhhjhjg} and
\er{jljjjjkhhhjhjhjhljjjljljkjk}. Then, inserting
\er{jljjjjkhhhjhjhjhljjjljljouukhh} and
\er{jljjjjkhhhjhjhjhljjjljljouukhhjhjg} into
\er{jljjjjkhhhjhjhjhljjjljljkjk} we deduce that
\begin{equation}
\label{jljjjjkhhhjhjhjhljjjljljouukhhkukyhjgj}\lim\limits_{|\vec
x|\to+\infty}\vec k(\vec x,t)=0\quad\quad\forall t\quad\quad\text{(
in the given system ${(*)}$)}\,.
\end{equation}
However, since $\vec k$ is a \underline{generally trivial}
speed-like vector field, by transformation law \er{NoIn5redbstr} in
the Definition \ref{bggghghgj} we deduce that for every fixed
instant of time $t$ $\vec k$ depends on $\vec x$
\underline{linearly} in \underline{every} cartesian coordinate
system. Plugging this fact with
\er{jljjjjkhhhjhjhjhljjjljljouukhhkukyhjgj} together gives:
\begin{equation}\label{vfyutuyfffhhgfhgfh33iy33yhijhjkhkhk}
\vec k(\vec x,t)=0\quad\quad\forall(\vec x,t)
\in\R^3\times[t_0,+\infty)\quad\quad\text{( in the given system
${(*)}$)}\,.
\end{equation}
So $\vec k$ is defined uniquely (exactly as in
\er{vfyutuyfffhhgfhgfh33iy33yhijhj}). Therefore, by
\er{jljjjjkhhhjhjhjhljjjljljkjk} and
\er{vfyutuyfffhhgfhgfh33iy33yhijhjkhkhk} we obtain
\begin{equation}\label{vfyutuyfffhhgfhgfh33iyjkkj33hjjggjhjjjhjjhkh}
\vec h(\vec x,t)=\vec v(\vec x,t)\quad\quad\forall(\vec x,t)
\in\R^3\times[t_0,+\infty)\quad\quad\text{( in the given system
${(*)}$)}\,.
\end{equation}
So $\vec h$ is also defined uniquely (exactly as in
\er{vfyutuyfffhhgfhgfh33iyjkkj33hjjggjhjj}). This completes the
proof.
\end{proof}

\begin{definition}\label{jhjkhhjhjhgjggjkkljkljklkhoiuiiy}
We say that two (inertial or non-inertial) cartesian coordinate
systems $(*)$ and $(**)$ are \underline{equivalent} if the change of
coordinates from the system $(*)$ to the system $(**)$ is given by
\begin{equation}\label{noninchgravortbstrjkjkjggg33jklkhjhjhuiui}
\begin{cases}
\vec x'=B\cdot\vec x+\vec c,\\
t'=t,
\end{cases}
\end{equation}
where $B\in SO(3)$ is a constant (independent of time) rotation and
$\vec c\in \mathbb{R}^3$ is a constant (independent of time) vector.
In other words, we obtain system $(**)$ from system $(*)$ just by a
constant (independent of time) rotation in $\mathbb{R}^3$ and/or by
a constant (independent of time) shift of the origin in
$\mathbb{R}^3$.
\end{definition}
\begin{lemma}\label{jkhjhjggjghkhjhhgk}
Let $\vec v:=\vec v(\vec x,t):\mathbb{R}^3\times
[t_0,+\infty)\to\mathbb{R}^3$ be a speed-like vector field and let
$\vec h:=\vec h(\vec x,t):\mathbb{R}^3\times
[t_0,+\infty)\to\mathbb{R}^3$ be a proper vector field. Moreover,
let $(*)$ and $(**)$ by two \underline{equivalent} cartesian
coordinate systems. Then we have
\begin{equation}\label{noninchgravortbstrjkjkjggg33jklkhjhjhuiuijhjhjh}
\vec v(\vec x,t)=\vec h(\vec x,t)\quad\quad\forall(\vec
x,t)\in\mathbb{R}^3\times [t_0,+\infty)\quad\quad\quad\text{in the
system $(*)$},
\end{equation}
if and only if
\begin{equation}\label{noninchgravortbstrjkjkjggg33jklkhjhjhuiuijhjhjhkjhjhjhgh}
\vec v'(\vec x',t')=\vec h'(\vec x',t')\quad\quad\forall(\vec
x',t')\in\mathbb{R}^3\times [t_0,+\infty)\quad\quad\quad\text{in the
system $(**)$}.
\end{equation}
In particular,
\begin{equation}\label{noninchgravortbstrjkjkjggg33jklkhjhjhuiuijhjhjh44}
\vec v(\vec x,t)=0\quad\quad\forall(\vec x,t)\in\mathbb{R}^3\times
[t_0,+\infty)\quad\quad\quad\text{in the system $(*)$},
\end{equation}
if and only if
\begin{equation}\label{noninchgravortbstrjkjkjggg33jklkhjhjhuiuijhjhjhkjhjhjhgh44}
\vec v'(\vec x',t')=0\quad\quad\forall(\vec
x',t')\in\mathbb{R}^3\times [t_0,+\infty)\quad\quad\quad\text{in the
system $(**)$}.
\end{equation}
\end{lemma}
\begin{proof}
Obviously, by Definition \ref{jhjkhhjhjhgjggjkkljkljklkhoiuiiy} the
change of coordinates from system $(*)$ to system $(**)$ is given
by:
\begin{equation}\label{noninchgravortbstrjkjkjggg33jklkhjhjhuiuijkjkkjkkh}
\begin{cases}
\vec x'=A(t)\cdot\vec x+\vec z(t),\\
t'=t,
\end{cases}
\end{equation}
where $A(t):=B\in SO(3)$ is a constant (independent of time)
rotation and $\vec z(t):=\vec c\in \mathbb{R}^3$ is a constant
(independent of time) vector. In particular,
\begin{equation}\label{noninchgravortbstrjkjkjggg33jklkhjhjhuiuijkjkkjkkhkljjkljkl}
\begin{cases}
\frac{dA}{dt}(t)=0\quad\quad\forall t,\\
\frac{d\vec z}{dt}(t)=0\quad\quad\forall t,
\end{cases}
\end{equation}
and so in this case the law of transformation of a speed-like vector
field in \er{NoIn5redbstr} of Definition \ref{bggghghgj} coincides
with the law of transformation of a proper vector field in
\er{uguyytfddddgghjjg} of Definition \ref{bggghghgj}.
\end{proof}

\begin{definition}\label{jhjkhhjhjhgjg;kkljkl}
Let $\vec v:=\vec v(\vec x,t):\mathbb{R}^3\times
[t_0,+\infty)\to\mathbb{R}^3$ be a speed-like smooth vector field
and let $(*)$ be some (inertial or non-inertial) cartesian
coordinate system. We say that $\vec v$ is \underline{asymptotically
acceptable} in the coordinate system $(*)$, if there exist
continuous mappings $B(t):[t_0,+\infty)\to\mathbb{R}^{3\times 3}$
and $\vec l(t):[t_0,+\infty)\to\mathbb{R}^{3}$, such that in the
coordinate system $(*)$ we have
\begin{equation}
\label{jljjjjkhhhhkhhjhjkjkjkl}
\begin{cases}
\lim\limits_{|\vec x|\to+\infty}\left(d_{\vec x}\vec v(\vec
x,t)\right)
=B(t)\quad\quad\quad\forall t\in [t_0,+\infty)\,,\\
\lim\limits_{|\vec x|\to+\infty}\left(d_{\vec x}\vec v(\vec
x,t)+\left\{d_{\vec x}\vec v(\vec
x,t)\right\}^T\right)=0\quad\quad\quad\forall t\in [t_0,+\infty)\,,
\\
\lim\limits_{|\vec x|\to+\infty}\left(\vec v(\vec x,t)-B(t)\cdot\vec
x\right)=\vec l(t)\quad\quad\quad\forall t\in [t_0,+\infty)\,.
\end{cases}
\end{equation}
\end{definition}
\begin{proposition}
\label{jhjkhhjhjhgjgjjkjkjklljj} Let $\vec v:=\vec v(\vec
x,t):\mathbb{R}^3\times [t_0,+\infty)\to\mathbb{R}^3$ be a
speed-like smooth vector field. Then the following three conditions
are \underline{equivalent}:
\begin{itemize}
\item[{\bf (i)}]
There exists a (inertial or non-inertial) cartesian coordinate
system $(*)$  where the vector field $\vec v$ is asymptotically
acceptable.
\item[{\bf (ii)}]
The vector field $\vec v$ is asymptotically acceptable in
\underline{every} (inertial or non-inertial) cartesian coordinate
system.
\item[{\bf (iii)}]
There exists a (inertial or non-inertial) cartesian coordinate
system ($\{0\}$) in which we have \begin{align}
\label{jljjjjkhhhhkhhjhjkjkjkljojkjkjkhj1} \lim\limits_{|\vec
x|\to+\infty}\left(d_{\vec x}\vec v(\vec x,t)\right)
=0\quad\quad\quad\forall t\in [t_0,+\infty)\,,\\
\label{jljjjjkhhhhkhhjhjkjkjkljojkjkjkhj2} \lim\limits_{|\vec
x|\to+\infty}\vec v(\vec x,t)=0\quad\quad\quad\forall t\in
[t_0,+\infty)\,.
\end{align}
Moreover, in that case the cartesian coordinate system
where \er{jljjjjkhhhhkhhjhjkjkjkljojkjkjkhj2} holds, is unique, up
to an equivalence, in other words it is unique, up to a constant
(independent of time) rotation in $\mathbb{R}^3$ and/or up to a
constant (independent of time) shift of the origin in
$\mathbb{R}^3$.
\end{itemize}
\end{proposition}
\begin{proof}
Consider some cartesian
coordinate system $(*)$ where we have \er{jljjjjkhhhhkhhjhjkjkjkl}:
\begin{equation}
\label{jljjjjkhhhhkhhjhjkjkjklkhhhk}
\begin{cases}
\lim\limits_{|\vec x|\to+\infty}\left(d_{\vec x}\vec v(\vec
x,t)\right)
=B(t)\quad\quad\quad\forall t\in [t_0,+\infty)\,,\\
\lim\limits_{|\vec x|\to+\infty}\left(d_{\vec x}\vec v(\vec
x,t)+\left\{d_{\vec x}\vec v(\vec
x,t)\right\}^T\right)=0\quad\quad\quad\forall t\in [t_0,+\infty)\,,
\\
\lim\limits_{|\vec x|\to+\infty}\left(\vec v(\vec x,t)-B(t)\cdot\vec
x\right)=\vec l(t)\quad\quad\quad\forall t\in [t_0,+\infty)\,,
\end{cases}
\end{equation}
In particular, by the first and the second equation in
\er{jljjjjkhhhhkhhjhjkjkjklkhhhk} we easily deduce that
\begin{equation}
\label{jljjjjkhhhhkhhjhjkjkjklkhhhkjkk} B^T(t)
=-B(t)\quad\quad\quad\forall t\in [t_0,+\infty).
\end{equation}
Next, let $(**)$ be another cartesian coordinate system, such that
the change of the coordinate system $(*)$ to the cartesian
coordinate system $(**)$ is given by:
\begin{equation}\label{noninchgravortbstrjgghguittu2hhhhhjkljjkhkh}
\begin{cases}
\vec x'=A(t)\cdot\vec x+\vec z(t),\\
t'=t,
\end{cases}
\end{equation}
where $A(t)\in SO(3)$ is a rotation. Then since $\vec v$ is a
speed-like vector field, as before, we have
\begin{equation}
\label{jljjjjkhhhhkhhjhjkhkjj}
\vec v'(\vec x',t')=A(t)\cdot \vec
v(\vec x,t)+
\frac{d A}{dt}(t)\cdot\vec x+
\frac{d\vec z}{dt}(t).
\end{equation}
In particular, as before, we obtain
\begin{equation}
\label{jljjjjkhhhhkhhjhjkhkjjjyhjggh} d_{\vec x'}\vec v'(\vec
x',t')=A(t)\cdot \left(d_{\vec x}\vec v(\vec x,t)\right)\cdot
A^{-1}(t)+
\frac{d A}{dt}(t)\cdot A^{-1}(t)=A(t)\cdot \left(d_{\vec x}\vec
v(\vec x,t)\right)\cdot A^{T}(t)+
\frac{d A}{dt}(t)\cdot A^{T}(t).
\end{equation}
Therefore, by the first equation in
\er{jljjjjkhhhhkhhjhjkjkjklkhhhk} we deduce
\begin{multline}
\label{jljjjjkhhhhkhhjhjkjkjklkhhhkknjnkkkjoj} \lim\limits_{|\vec
x'|\to+\infty}\left(d_{\vec x'}\vec v'(\vec
x',t')\right)=\lim\limits_{|\vec x|\to+\infty}\left(A(t)\cdot
\left(d_{\vec x}\vec v(\vec x,t)\right)\cdot A^{T}(t)+ \frac{d
A}{dt}(t)\cdot A^{T}(t)\right)\\=A(t)\cdot B(t)\cdot
A^{T}(t)+\frac{d A}{dt}(t)\cdot
A^{T}(t)=B'(t')\quad\quad\quad\forall t'=t\in [t_0,+\infty),
\end{multline}
where we denote
\begin{equation}
\label{jljjjjkhhhhkhhjhjkjkjklkhhhkknjnkkkjojjkjk} B'(t'):=A(t)\cdot
B(t)\cdot A^{T}(t)+\frac{d A}{dt}(t)\cdot
A^{T}(t)\quad\quad\quad\forall t'=t\in [t_0,+\infty).
\end{equation}
Then, by the third equation in \er{jljjjjkhhhhkhhjhjkjkjklkhhhk}
together with \er{jljjjjkhhhhkhhjhjkjkjklkhhhkknjnkkkjojjkjk},
\er{jljjjjkhhhhkhhjhjkhkjj} and
\er{noninchgravortbstrjgghguittu2hhhhhjkljjkhkh} we infer:
\begin{multline}
\label{jljjjjkhhhhkhhjhjkjkjklkhhhkknjnkkkjoj1} \lim\limits_{|\vec
x'|\to+\infty}\left(\vec v'(\vec x',t')-B'(t')\cdot\vec x'\right)=
\lim\limits_{|\vec x|\to+\infty}\Bigg\{\left( A(t)\cdot \vec v(\vec
x,t)+ \frac{d A}{dt}(t)\cdot\vec x+ \frac{d\vec
z}{dt}(t)\right)\\-\left(A(t)\cdot B(t)\cdot A^{T}(t)+\frac{d
A}{dt}(t)\cdot A^{T}(t)\right)\cdot\left(A(t)\cdot\vec x+\vec
z(t)\right)\Bigg\}\\= \lim\limits_{|\vec x|\to+\infty}\Bigg\{\left(
A(t)\cdot \vec v(\vec x,t)+ \frac{d A}{dt}(t)\cdot\vec x+
\frac{d\vec z}{dt}(t)\right)\\-\left(A(t)\cdot B(t)\cdot
A^{-1}(t)+\frac{d A}{dt}(t)\cdot
A^{-1}(t)\right)\cdot\left(A(t)\cdot\vec x\right)-\left(A(t)\cdot
B(t)\cdot A^{T}(t)+\frac{d A}{dt}(t)\cdot A^{T}(t)\right)\cdot\vec
z(t)\Bigg\}\\=\lim\limits_{|\vec x|\to+\infty}\Bigg\{\left(
A(t)\cdot  \left(\vec v(\vec x,t)-B(t)\cdot\vec x\right)+
\frac{d\vec z}{dt}(t)\right)-\left(A(t)\cdot B(t)\cdot
A^{T}(t)+\frac{d A}{dt}(t)\cdot A^{T}(t)\right)\cdot\vec z(t)\Bigg\}
\\
=\left(A(t)\cdot  \vec l(t)+ \frac{d\vec
z}{dt}(t)\right)-\left(A(t)\cdot B(t)\cdot A^{T}(t)+\frac{d
A}{dt}(t)\cdot A^{T}(t)\right)\cdot\vec z(t):=\vec
 l'(t')\quad\;\;\forall t'=t\in [t_0,+\infty),
\end{multline}
where we denote
\begin{equation}
\label{jljjjjkhhhhkhhjhjkjkjklkhhhkknjnkkkjojjkjkhgjghjk} \vec
l'(t'):=\left(A(t)\cdot  \vec l(t)+ \frac{d\vec
z}{dt}(t)\right)-\left(A(t)\cdot B(t)\cdot A^{T}(t)+\frac{d
A}{dt}(t)\cdot A^{T}(t)\right)\cdot\vec z(t)\quad\;\;\forall t'=t\in
[t_0,+\infty).
\end{equation}
On the other hand, by  {\bf(iii)} in Proposition \ref{yghgjtgyrtrt}
together with the second equation in
\er{jljjjjkhhhhkhhjhjkjkjklkhhhk} we obtain
\begin{multline}
\label{jljjjjkhhhhkhhjhjkjkjklkhhhkjkhkhj} \lim\limits_{|\vec
x'|\to+\infty}\left(d_{\vec x'}\vec v'(\vec x',t')+\left\{d_{\vec
x'}\vec v'(\vec
x',t')\right\}^T\right)\\=A(t)\cdot\Bigg\{\lim\limits_{|\vec
x|\to+\infty}\left(d_{\vec x}\vec v(\vec x,t)+\left\{d_{\vec x}\vec
v(\vec x,t)\right\}^T\right)\Bigg\}\cdot
A^T(t)=0\quad\quad\quad\forall t=t'\in [t_0,+\infty)\,.
\end{multline}
So, by \er{jljjjjkhhhhkhhjhjkjkjklkhhhkknjnkkkjoj},
\er{jljjjjkhhhhkhhjhjkjkjklkhhhkknjnkkkjoj1} and
\er{jljjjjkhhhhkhhjhjkjkjklkhhhkjkhkhj} we obtain
\begin{equation}
\label{jljjjjkhhhhkhhjhjkjkjklkhhhkmnhbhkh}
\begin{cases}
\lim\limits_{|\vec x'|\to+\infty}\left(d_{\vec x'}\vec v'(\vec
x',t')\right)=B'(t')\quad\quad\quad\forall t'\in [t_0,+\infty),
\\
\lim\limits_{|\vec x'|\to+\infty}\left(d_{\vec x'}\vec v'(\vec
x',t')+\left\{d_{\vec x'}\vec v'(\vec
x',t')\right\}^T\right)=0\quad\quad\quad\forall t'\in [t_0,+\infty),
\\
\lim\limits_{|\vec x'|\to+\infty}\left(\vec v'(\vec
x',t')-B'(t')\cdot\vec x'\right)=\vec l'(t')\quad\;\;\forall t'\in
[t_0,+\infty),
\end{cases}
\end{equation}
and therefore if $\vec v$ is \underline{asymptotically acceptable}
in the coordinate system {\bf (*)} then, by
\er{jljjjjkhhhhkhhjhjkjkjklkhhhkmnhbhkh} $\vec v'$ is also
\underline{asymptotically acceptable} in the coordinate system {\bf
(**)} and implication $\text{\bf (i)}\Longleftrightarrow\text{\bf
(ii)}$ follows. Moreover, by
\er{jljjjjkhhhhkhhjhjkjkjklkhhhkknjnkkkjoj} and
\er{jljjjjkhhhhkhhjhjkjkjklkhhhkknjnkkkjoj1} we deduced
\begin{equation}
\label{jljjjjkhhhhkhhjhjkjkjklkhhhkknjnkkkjojjkhkhkhthr}
\lim\limits_{|\vec x'|\to+\infty}\left(d_{\vec x'}\vec v'(\vec
x',t')\right)=B'(t'):=\left(\frac{d A}{dt}(t)+A(t)\cdot
B(t)\right)\cdot A^T(t)\quad\quad\quad\forall t'=t\in [t_0,+\infty),
\end{equation}
and
\begin{multline}
\label{jljjjjkhhhhkhhjhjkjkjklkhhhkmnhbhkhnhgjhjggjhjg}
\lim\limits_{|\vec x'|\to+\infty}\left(\vec v'(\vec
x',t')-B'(t')\cdot\vec x'\right)=\vec l'(t'):=\\ \frac{d\vec
z}{dt}(t)-\left(A(t)\cdot B(t)\cdot A^{T}(t)+\frac{d A}{dt}(t)\cdot
A^{T}(t)\right)\cdot\vec z(t)+A(t)\cdot  \vec l(t) \quad\;\;\forall
t'=t\in [t_0,+\infty),
\end{multline}
where, by \er{jljjjjkhhhhkhhjhjkjkjklkhhhkjkk} $B(t)$ satisfies
\begin{equation}
\label{jljjjjkhhhhkhhjhjkjkjklkhhhkjkkfhf} B^T(t)
=-B(t)\quad\quad\quad\forall t\in [t_0,+\infty).
\end{equation}
Next, assume that $\tau_0\in [t_0,+\infty)$ and a smooth mappings
$A_0(t):[t_0,+\infty)\to\mathbb{R}^{3\times 3}$ and $\vec
z_0(t):[t_0,+\infty)\to\mathbb{R}^{3}$ satisfy
\begin{equation}
\label{jljjjjkhhhhkhhjhjkjkjklkhhhkknjnkkkjojjkhkhkhthrkljljhg}
\begin{cases}
\frac{d A_0}{dt}(t)=-A_0(t)\cdot B(t)\quad\quad\quad\forall t\in
[t_0,+\infty)\\
A_0(\tau_0)=I,
\end{cases}
\end{equation}
and
\begin{equation}
\label{jljjjjkhhhhkhhjhjkjkjklkhhhkmnhbhkhnhgjhjggjhjggjgjgh}
\frac{d\vec z_0}{dt}(t)=\left(A(t)\cdot B(t)\cdot A^{T}(t)+\frac{d
A}{dt}(t)\cdot A^{T}(t)\right)\cdot\vec z_0(t)-A(t)\cdot  \vec l(t)
\quad\;\;\forall t=t\in [t_0,+\infty),
\end{equation}
where $I\in \mathbb{R}^{3\times 3}$ is the identity matrix. By the
Theory of Ordinary Differential Equations it is well known that such
a mappings $A_0(t)$ and $\vec z_0(t)$ are indeed exist. Then by
\er{jljjjjkhhhhkhhjhjkjkjklkhhhkknjnkkkjojjkhkhkhthrkljljhg} and
\er{jljjjjkhhhhkhhjhjkjkjklkhhhkjkkfhf} we deduce,
\begin{equation}
\label{jljjjjkhhhhkhhjhjkjkjklkhhhkknjnkkkjojjkhkhkhthrkljljhjgj}
\frac{d }{dt}\left\{A^T_0(t)\right\}=-B^T(t)\cdot A^T_0=B(t)\cdot
\left\{A^T_0(t)\right\}\quad\quad\quad\forall t\in [t_0,+\infty),
\end{equation}
and so
\begin{equation}
\label{jljjjjkhhhhkhhjhjkjkjklkhhhkknjnkkkjojjkhkhkhthrkljljhgjkjjh}
\begin{cases}
\frac{d }{dt}\left\{A^T_0(t)\right\}=B(t)\cdot
\left\{A^T_0(t)\right\}\quad\quad\quad\forall t\in [t_0,+\infty)
\\ A^T_0(\tau_0)=I.
\end{cases}
\end{equation}
In particular, again by the Theory of Ordinary Differential
Equations using
\er{jljjjjkhhhhkhhjhjkjkjklkhhhkknjnkkkjojjkhkhkhthrkljljhgjkjjh} we
deduce that
\begin{equation}
\label{jljjjjkhhhhkhhjhjkjkjklkhhhkknjnkkkjojjkhkhkhthrkljljhjgjhjh}
\det A_0^T(t)>0\quad\quad\quad\forall t\in [t_0,+\infty)
\end{equation}
and then also
\begin{equation}
\label{jljjjjkhhhhkhhjhjkjkjklkhhhkknjnkkkjojjkhkhkhthrkljljhjgjhjhkkjkljj}
\det A_0(t)>0\quad\quad\quad\forall t\in [t_0,+\infty).
\end{equation}
So, $A_0(t)$ and $A^T_0(t)$ are invertible matrices. On the other
hand, since differentiating the identity $$A^{-1}_0(t)\cdot
A_0(t)=I\quad\quad\quad\forall t'=t\in [t_0,+\infty)$$ gives
\begin{equation*}
\left(\frac{d}{dt}\left\{A_0^{-1}(t)\right\}\right)\cdot
A_0(t)+A_0^{-1}(t)\cdot\frac{d A_0}{dt}(t)=0\quad\quad\quad\forall
t'=t\in [t_0,+\infty),
\end{equation*}
we deduce
\begin{equation}
\label{jljjjjkhhhhkhhjhjkjkjklkhhhkknjnkkkjojjkhkhkhthrgghghjkjhmnb}
\frac{d}{dt}\left\{A_0^{-1}(t)\right\}=-A_0^{-1}(t)\cdot\frac{d
A_0}{dt}(t)\cdot A_0^{-1}(t)\quad\quad\quad\forall t'=t\in
[t_0,+\infty).
\end{equation}
Thus, by
\er{jljjjjkhhhhkhhjhjkjkjklkhhhkknjnkkkjojjkhkhkhthrkljljhg} and
\er{jljjjjkhhhhkhhjhjkjkjklkhhhkknjnkkkjojjkhkhkhthrgghghjkjhmnb} we
deduce
\begin{equation}
\label{jljjjjkhhhhkhhjhjkjkjklkhhhkknjnkkkjojjkhkhkhthrkljljhjgjjkhjhgghghjkjkhkhk}
\frac{d}{dt}\left\{A^{-1}_0(t)\right\}=-A_0^{-1}(t)\cdot \left(-
A_0(t)\cdot B(t)\right)\cdot A_0^{-1}(t)=B(t)\cdot A_0^{-1}(t)
\quad\quad\forall t\in [t_0,+\infty),
\end{equation}
and so
\begin{equation}
\label{jljjjjkhhhhkhhjhjkjkjklkhhhkknjnkkkjojjkhkhkhthrkljljhgjkjjh222}
\begin{cases}
\frac{d }{dt}\left\{A^{-1}_0(t)\right\}=B(t)\cdot
\left\{A^{-1}_0(t)\right\}\quad\quad\quad\forall t\in [t_0,+\infty)
\\ A^{-1}_0(\tau_0)=I.
\end{cases}
\end{equation}
Therefore, by Uniqueness Theorem for the Cauchy Problem for Ordinary
Differential Equations,
\er{jljjjjkhhhhkhhjhjkjkjklkhhhkknjnkkkjojjkhkhkhthrkljljhgjkjjh}
and
\er{jljjjjkhhhhkhhjhjkjkjklkhhhkknjnkkkjojjkhkhkhthrkljljhgjkjjh222}
together imply that
\begin{equation}
\label{jljjjjkhhhhkhhjhjkjkjklkhhhkknjnkkkjojjkhkhkhthrkljljhjgjjkhjhgghghjkjkhkhkkhhk}
A^{-1}_0(t)=A^{T}_0(t)\quad\quad\quad\forall t\in [t_0,+\infty).
\end{equation}
So by
\er{jljjjjkhhhhkhhjhjkjkjklkhhhkknjnkkkjojjkhkhkhthrkljljhjgjhjhkkjkljj}
and
\er{jljjjjkhhhhkhhjhjkjkjklkhhhkknjnkkkjojjkhkhkhthrkljljhjgjjkhjhgghghjkjkhkhkkhhk}
we deduce that $A_0(t)$, which solves
\er{jljjjjkhhhhkhhjhjkjkjklkhhhkknjnkkkjojjkhkhkhthrkljljhg},
necessary satisfies $A_0(t)\in SO(3)$. In other words $A_0(t)$ which
solves \er{jljjjjkhhhhkhhjhjkjkjklkhhhkknjnkkkjojjkhkhkhthrkljljhg}
is a \underline{rotation}.

Therefore, there exists some particular cartesian coordinate system
$(**)$, such that the change of the coordinate system $(*)$ to the
cartesian coordinate system $(**)$ is given by
\er{noninchgravortbstrjgghguittu2hhhhhjkljjkhkh} with $A(t)=A_0(t)$
and $\vec z(t)=\vec z_0(t)$:
\begin{equation}\label{noninchgravortbstrjgghguittu2hhhhhjkljjkhkhkhh}
\begin{cases}
\vec x'=A(t)\cdot\vec x+\vec z(t)=A_0(t)\cdot\vec x+\vec z_0(t),\\
t'=t,
\end{cases}
\end{equation}
and moreover, by
\er{jljjjjkhhhhkhhjhjkjkjklkhhhkknjnkkkjojjkhkhkhthrkljljhg} for
this system transformation we have
\begin{equation}
\label{jljjjjkhhhhkhhjhjkjkjklkhhhkknjnkkkjojjkhkhkhthrkljljhgukkuuyk}
\frac{d A}{dt}(t)+A(t)\cdot B(t)=0\quad\quad\quad\forall t\in
[t_0,+\infty),
\end{equation}
and by \er{jljjjjkhhhhkhhjhjkjkjklkhhhkmnhbhkhnhgjhjggjhjggjgjgh} we
have
\begin{equation}
\label{jljjjjkhhhhkhhjhjkjkjklkhhhkmnhbhkhnhgjhjggjhjggjgjghyjh}
\frac{d\vec z}{dt}(t)-\left(A(t)\cdot B(t)\cdot A^{T}(t)+\frac{d
A}{dt}(t)\cdot A^{T}(t)\right)\cdot\vec z(t)+A(t)\cdot  \vec l(t)=0.
\quad\;\;\forall t\in [t_0,+\infty),
\end{equation}
Thus, inserting
\er{jljjjjkhhhhkhhjhjkjkjklkhhhkknjnkkkjojjkhkhkhthrkljljhgukkuuyk}
into \er{jljjjjkhhhhkhhjhjkjkjklkhhhkknjnkkkjojjkhkhkhthr} we deduce
\begin{equation}
\label{jljjjjkhhhhkhhjhjkjkjklkhhhkknjnkkkjojjkhkhkhthrggg}
\lim\limits_{|\vec x'|\to+\infty}\left(d_{\vec x'}\vec v'(\vec
x',t')\right)=0\quad\quad\quad\forall t'\in
[t_0,+\infty)\quad\quad\text{(in the chosen system  $(**)$)}.
\end{equation}
Moreover, inserting
\er{jljjjjkhhhhkhhjhjkjkjklkhhhkmnhbhkhnhgjhjggjhjggjgjghyjh} into
\er{jljjjjkhhhhkhhjhjkjkjklkhhhkmnhbhkhnhgjhjggjhjg} we deduce
\begin{multline}
\label{jljjjjkhhhhkhhjhjkjkjklkhhhkmnhbhkhnhgjhjggjhjgjggh}
\lim\limits_{|\vec x'|\to+\infty}\vec v'(\vec
x',t')=\lim\limits_{|\vec x'|\to+\infty}\left(\vec v'(\vec
x',t')-B'(t')\cdot\vec x'\right)=0\\ \quad\quad\quad\quad\forall
t'\in [t_0,+\infty)\quad\quad\quad\text{(in the chosen system
$(**)$)}.
\end{multline}
and by \er{jljjjjkhhhhkhhjhjkjkjklkhhhkknjnkkkjojjkhkhkhthrggg}
together with
\er{jljjjjkhhhhkhhjhjkjkjklkhhhkmnhbhkhnhgjhjggjhjgjggh} implication
$\text{\bf (i)}\Longleftrightarrow\text{\bf (iii)}$ also follows.

Finally, consider two cartesian coordinate systems $(\{1\})$ and
$(\{2\})$ so that the change of coordinate system $(\{1\})$ to the
system $(\{2\})$ can be written as:
\begin{equation}\label{noninchgravortbstrjkjkjggg33hhjjh}
\begin{cases}
\vec x^{(2)}=A(t^{(1)})\cdot\vec x^{(1)}+\vec z(t^{(1)}),\\
t^{(2)}=t^{(1)},
\end{cases}
\end{equation}
where $A(t)\in SO(3)$ is a rotation. Moreover, assume that the field
$\vec v$ in both these systems satisfies
\begin{align}
\label{jljjjjkhhhhkhhjhjkjkjklkhhhkmnhbhkhnhgjhjggjhjgjgghjhjhjh}
\lim\limits_{|\vec x^{(1)}|\to+\infty}\vec v^{(1)}(\vec
x^{(1)},t^{(1)})=0 \quad\quad\quad\quad\forall t^{(1)}=t\in
[t_0,+\infty),\\
\label{jljjjjkhhhhkhhjhjkjkjklkhhhkmnhbhkhnhgjhjggjhjgjgghjhjhjh1}
\lim\limits_{|\vec x^{(2)}|\to+\infty}\vec v^{(2)}(\vec
x^{(2)},t^{(2)})=0 \quad\quad\quad\quad\forall t^{(2)}=t\in
[t_0,+\infty).
\end{align}
Then, since $\vec v$ is a speed-like vector field, by
\er{NoIn5redbstr} in Definition \ref{bggghghgj} we deduce that
\begin{equation}
\label{jljjjjkhhhhkhhjhjkhkjjjkljjhjh} \vec v^{(2)}(\vec
x^{(2)},t)=A(t)\cdot \vec v^{(1)}(\vec x^{(1)},t)+
\frac{d A}{dt}(t)\cdot\vec x^{(1)}+
\frac{d\vec z}{dt}(t)\quad\quad\quad\quad\forall t\in [t_0,+\infty).
\end{equation}
Thus, by \er{jljjjjkhhhhkhhjhjkhkjjjkljjhjh} two equalities
\er{jljjjjkhhhhkhhjhjkjkjklkhhhkmnhbhkhnhgjhjggjhjgjgghjhjhjh} and
\er{jljjjjkhhhhkhhjhjkjkjklkhhhkmnhbhkhnhgjhjggjhjgjgghjhjhjh1} can
be satisfied simultaneously if and only if we have
\begin{equation}
\label{jljjjjkhhhhkhhjhjkhkjjjkljjhjhjgg}\frac{d
A}{dt}(t)=0\quad\text{and}\quad \frac{d\vec
z}{dt}(t)=0\quad\quad\quad\quad\forall t\in [t_0,+\infty),
\end{equation}
and so this happens if and only if $A$ and $\vec z$ in
\er{noninchgravortbstrjkjkjggg33hhjjh} are constants (independent of
time), that means two systems $(\{1\})$ and $(\{2\})$ are
\underline{equivalent}.
This completes the proof.
\end{proof}
\begin{corollary}
\label{jhjkhhjhjhgjgjjkjkjhhjhCC} Let $\vec v:=\vec v(\vec
x,t):\mathbb{R}^3\times[t_0,+\infty)\to\mathbb{R}^3$ be an
\underline{asymptotically acceptable} speed-like vector field. Then
there exist a uniquely defined \underline{generally trivial
speed-like} vector field $\,\vec k:=\vec k(\vec
x,t):\mathbb{R}^3\times[t_0,+\infty)\to\mathbb{R}^3$ and a uniquely
defined \underline{proper} vector field $\,\vec h:=\vec h(\vec
x,t):\mathbb{R}^3\times[t_0,+\infty)\to\mathbb{R}^3$, so that in
every cartesian coordinate system we have
\begin{equation}
\label{jljjjjkhhhjhjhjhljjjljljouukhhjhjgCC} \lim\limits_{|\vec
x|\to+\infty}\vec h(\vec x,t)=0\quad\quad\forall t\,,
\end{equation}
and in every cartesian coordinate system we can decompose vector
field $\vec v$ as:
\begin{equation}
\label{jljjjjkhhhjhjhjhljjjljljkjkCC} \vec v(\vec x,t)=\vec h(\vec
x,t)+\vec k(\vec x,t)\quad\quad\quad\quad\forall (\vec
x,t)\in\mathbb{R}^3\times[t_0,+\infty).
\end{equation}
Moreover, in addition, the proper vector field $\,\vec h$
necessarily satisfies
\begin{equation}
\label{jljjjjkhhhjhjhjhljjjljljouukhhjhjgCCkhjkjhCC}
\lim\limits_{|\vec x|\to+\infty}\left(d_{\vec x}\vec h(\vec
x,t)\right)=0\quad\quad\forall t\quad\quad\quad\text{(in every
cartesian coordinate system)}\,.
\end{equation}
\end{corollary}
\begin{proof}
The result is a direct consequence of both Propositions
\ref{jhjkhhjhjhgjgjjkjkjhhjh} and \ref{jhjkhhjhjhgjgjjkjkjklljj},
where only identity
\er{jljjjjkhhhjhjhjhljjjljljouukhhjhjgCCkhjkjhCC} require to be
proved. In order to prove
\er{jljjjjkhhhjhjhjhljjjljljouukhhjhjgCCkhjkjhCC} consider the
system $(*)$ where we have
\begin{equation}\label{ghghfjdfjgdg}
\begin{cases}
\lim\limits_{|\vec x|\to+\infty}\left(d_{\vec x}\vec v(\vec
x,t)\right)
=0\quad\quad\quad\forall t\in [t_0,+\infty)\,,\\
\lim\limits_{|\vec x|\to+\infty}\vec v(\vec
x,t)=0\quad\quad\quad\forall t\in [t_0,+\infty)\,,
\end{cases}
\end{equation}
and consider the system $(**)$ where we have
\begin{equation}
\label{jljjjjkhhhjhjhjhljjjljljkjkCCjhgghjkk}
\vec k'(\vec
x',t')=0\quad\quad\quad\quad\forall (\vec
x',t')\in\mathbb{R}^3\times[t_0,+\infty).
\end{equation}
On the other hand, as in \er{jljjjjkhhhjhjhjhljjjljljouukhhjhjgCC}
and \er{jljjjjkhhhjhjhjhljjjljljkjkCC} in the system $(**)$ we also
have
\begin{equation}
\label{jljjjjkhhhjhjhjhljjjljljouukhhjhjgCChh} \lim\limits_{|\vec
x'|\to+\infty}\vec h'(\vec x',t')=0\quad\quad\forall t'\,,
\end{equation}
and
\begin{equation}
\label{jljjjjkhhhjhjhjhljjjljljkjkCChh} \vec v'(\vec x',t')=\vec
h'(\vec x',t')+\vec k'(\vec x',t')\quad\quad\quad\quad\forall (\vec
x',t')\in\mathbb{R}^3\times[t_0,+\infty).
\end{equation}
In particular, by \er{jljjjjkhhhjhjhjhljjjljljkjkCChh},
\er{jljjjjkhhhjhjhjhljjjljljkjkCCjhgghjkk} and
\er{jljjjjkhhhjhjhjhljjjljljouukhhjhjgCChh} together we deduce that
in system $(**)$ we have
\begin{equation}
\label{jljjjjkhhhjhjhjhljjjljljkjkCChhouu} \vec v'(\vec x',t')=\vec
h'(\vec x',t')\quad\quad\quad\quad\forall (\vec
x',t')\in\mathbb{R}^3\times[t_0,+\infty),
\end{equation}
and
\begin{equation}
\label{jljjjjkhhhjhjhjhljjjljljouukhhjhjgCChhljjkljkl}
\lim\limits_{|\vec x'|\to+\infty}\vec v'(\vec
x',t')=0\quad\quad\forall t'\,,
\end{equation}
Therefore, again using Proposition \ref{jhjkhhjhjhgjgjjkjkjklljj},
by the second equality in \er{ghghfjdfjgdg} together with
\er{jljjjjkhhhjhjhjhljjjljljouukhhjhjgCChhljjkljkl} we obtain that
systems $(*)$ and $(**)$ are \underline{equivalent}. By this
equivalence and by Lemma \ref{jkhjhjggjghkhjhhgk}, using
\er{jljjjjkhhhjhjhjhljjjljljkjkCChhouu} we deduce that in system
$(*)$ we also have
\begin{equation}
\label{jljjjjkhhhjhjhjhljjjljljkjkCChhouujkljkjk} \vec v(\vec
x,t)=\vec h(\vec x,t)\quad\quad\quad\quad\forall (\vec
x,t)\in\mathbb{R}^3\times[t_0,+\infty).
\end{equation}
Thus, by \er{jljjjjkhhhjhjhjhljjjljljkjkCChhouujkljkjk} and by the
first equality in \er{ghghfjdfjgdg} we finally deduce
\er{jljjjjkhhhjhjhjhljjjljljouukhhjhjgCCkhjkjhCC}.
\end{proof}

\begin{definition}\label{gjgjffjjjjjjjjjjjjjjjjjjjjjjhh}
Let $\mathcal{U}$ be a nonempty subset of the set of all inertial
and non-inertial cartesian coordinate systems. We say that
$\mathcal{U}$ is an \underline{extended Galilean group} if, given
arbitrary cartesian coordinate system $(*)$ inside the subset
$\mathcal{U}$, the following statement holds:
\begin{itemize}
\item
The arbitrary cartesian coordinate system $(**)$ belongs to the
subset $\mathcal{U}$ if and only if there exist a constant
(independent of time) rotation $B\in SO(3)$ and constant
(independent of time) vectors $\vec c\in \mathbb{R}^3$ and $\vec
w\in \mathbb{R}^3$, so that the change of coordinates from the
system $(*)$ to the system $(**)$ is given by
\begin{equation}\label{noninchgravortbstrjkjkjggg33jklkhjhjhuiuiljjljjkkl}
\begin{cases}
\vec x'=B\cdot\vec x+\vec c+\vec w t,\\
t'=t.
\end{cases}
\end{equation}
\end{itemize}
In other words, the system $(**)$ belongs to $\mathcal{U}$ if and
only if, up to equivalence of cartesian coordinate systems, the
system $(**)$ can be obtained from the system $(*)$ by the Galilean
Transformation
\begin{equation}\label{noninchgravortbstrjgghguittu1hjmjh}
\begin{cases}
\vec x'=\vec x+\vec wt,\\
t'=t.
\end{cases}
\end{equation}
\end{definition}
\begin{definition}\label{gjgjffjjjjjjjjjjjjjjjjjjjjjjhhr} Let
$\mathcal{V}$ be a nonempty subset of the set of all inertial and
non-inertial cartesian coordinate systems. We say that $\mathcal{V}$
is a \underline{non-rotational rigid motion group} if, given
arbitrary cartesian coordinate system $(*)$ inside the subset
$\mathcal{V}$, the following statement holds:
\begin{itemize}
\item
The arbitrary cartesian coordinate system $(**)$ belongs to the
subset $\mathcal{V}$ if and only if there exist a constant
(independent of time) rotation $B\in SO(3)$ and a smooth
vector-valued function of the time variable $\vec z(t):\mathbb{R}\to
\mathbb{R}^3$, so that the change of coordinates from the system
$(*)$ to the system $(**)$ is given by
\begin{equation}\label{noninchgravortbstrjkjkjggg33jklkhjhjhuiuiljjljjkklr}
\begin{cases}
\vec x'=B\cdot\vec x+\vec z(t),\\
t'=t.
\end{cases}
\end{equation}
\end{itemize}
\end{definition}
\begin{definition}\label{gjgjffjjjjjjjjjjjjjjjjjjjjjjhhrf}
Let $\mathcal{W}$ be a certain nonempty subset of coordinate
systems. We say that $\mathcal{W}$ is a \underline{general rigid
motion group} if, given arbitrary coordinate system $(*)$ inside the
subset $\mathcal{W}$, the following statement holds:
\begin{itemize}
\item
The arbitrary coordinate system $(**)$ belongs to the subset
$\mathcal{W}$ if and only if there exist a smooth (time dependent)
rotation $A(t):\mathbb{R}\to SO(3)$ and a smooth vector-valued
function of the time variable $\vec z(t):\mathbb{R}\to
\mathbb{R}^3$, so that the change of coordinates from the system
$(*)$ to the system $(**)$ is given by
\begin{equation}\label{noninchgravortbstrjkjkjggg33jklkhjhjhuiuiljjljjkklrf}
\begin{cases}
\vec x'=A(t)\cdot\vec x+\vec z(t),\\
t'=t.
\end{cases}
\end{equation}
\end{itemize}
In particular, the set of all \underline{cartesian} coordinate
systems is a general rigid motion group.
\end{definition}

\begin{definition}\label{jhjkhhjhjhgjg;kkljkljkkjlhlhhlh}
Let $\vec v:=\vec v(\vec x,t):\mathbb{R}^3\times
[t_0,+\infty)\to\mathbb{R}^3$ be an \underline{asymptotically
acceptable} speed-like smooth vector field (see Definition
\ref{jhjkhhjhjhgjg;kkljkl}).
\begin{itemize}
\item
We call that the cartesian coordinate system ($\{0\}$) is
\underline{preferable} for the vector field $\vec v$ if in the
system ($\{0\}$) we have
\begin{align} \label{jljjjjkhhhhkhhjhjkjkjkljojkjkjkhj1hh}
\lim\limits_{|\vec x|\to+\infty}\left(d_{\vec x}\vec v(\vec
x,t)\right)
=0\quad\quad\quad\forall t\in [t_0,+\infty)\,,\\
\label{jljjjjkhhhhkhhjhjkjkjkljojkjkjkhj2hh}
\lim\limits_{|\vec
x|\to+\infty}\vec v(\vec x,t)=0\quad\quad\quad\forall t\in
[t_0,+\infty)\,.
\end{align}
Note that by Proposition \ref{jhjkhhjhjhgjgjjkjkjklljj} the unique
(up to equivalence) preferable for the vector field $\vec v$
cartesian system ($\{0\}$) exists.
\item
We call that the cartesian coordinate system $(*)$ is
\underline{inertial} with respect to the vector field $\vec v$ if in
the system $(*)$ we have
\begin{align} \label{jljjjjkhhhhkhhjhjkjkjkljojkjkjkhj1hhj}
\lim\limits_{|\vec x|\to+\infty}\left(d_{\vec x}\vec v(\vec
x,t)\right)
=0\quad\quad\quad\forall t\in [t_0,+\infty)\,,\\
\label{jljjjjkhhhhkhhjhjkjkjkljojkjkjkhj2hhj} \lim\limits_{|\vec
x|\to+\infty}\vec v(\vec x,t)=\vec d\quad\quad\quad\forall t\in
[t_0,+\infty)\,,
\end{align}
where $\vec d\in \mathbb{R}^3$ is a constant (independent of $t$)
vector. Note that, in particular, by the definition, the preferable
for $\vec v$ cartesian system is also an inertial with respect to
$\vec v$  coordinate system.
\item
We call that the cartesian coordinate system $(*)$ is
\underline{non-rotating} with respect to the vector field $\vec v$
if in the system $(*)$ we have
\begin{equation} \label{jljjjjkhhhhkhhjhjkjkjkljojkjkjkhj1hhjhkkj}
\lim\limits_{|\vec x|\to+\infty}\left(d_{\vec x}\vec v(\vec
x,t)\right)=0\quad\quad\quad\forall t\in [t_0,+\infty)\,,
\end{equation}
(without specifying $\lim\limits_{|\vec x|\to+\infty}\vec v(\vec
x,t)$). Note that, in particular, by the definition, every inertial
with respect to $\vec v$ cartesian coordinate system is also a
non-rotating with respect to $\vec v$ coordinate system.
\end{itemize}
\end{definition}
\begin{proposition}\label{jghfgjdhdhg}
Let $\vec v:=\vec v(\vec x,t):\mathbb{R}^3\times
[t_0,+\infty)\to\mathbb{R}^3$ be an \underline{asymptotically
acceptable} speed-like smooth vector field (see Definition
\ref{jhjkhhjhjhgjg;kkljkl}). Furthermore, let $\mathcal{U}$ be the
set of all \underline{inertial} with respect to $\vec v$ cartesian
coordinate systems and let $\mathcal{V}$ be the set of all
\underline{non-rotating} with respect to $\vec v$ cartesian
coordinate systems (see Definition
\ref{jhjkhhjhjhgjg;kkljkljkkjlhlhhlh}). Then $\mathcal{U}$ is an
\underline{extended Galilean group} and  $\mathcal{V}$ is a
\underline{non-rotational rigid motion group} (see Definitions
\ref{gjgjffjjjjjjjjjjjjjjjjjjjjjjhh} and
\ref{gjgjffjjjjjjjjjjjjjjjjjjjjjjhhr}).
\end{proposition}
\begin{proof}
Consider the change of some cartesian coordinate system $(*)$ to
another cartesian coordinate system $(**)$ of the form:
\begin{equation}\label{noninchgravortbstrgjgghhhh}
\begin{cases}
\vec x'=A(t)\cdot\vec x+\vec z(t),\\
t'=t,
\end{cases}
\end{equation}
where $A(t)\in SO(3)$ and correspondingly to \er{NoIn5redbstr} we
have
\begin{equation}
\label{NoIn5redbstrhhhh}\vec v'(\vec x',t')=A(t)\cdot \vec v(\vec
x,t)+
\frac{d A}{dt}(t)\cdot\vec x+
\frac{d\vec z}{dt}(t).
\end{equation}
In particular, by \er{noninchgravortbstrgjgghhhh},
\er{NoIn5redbstrhhhh} and the Chain rule we deduce
\begin{equation}
\label{NoIn5redbstrhhhhjkhhk}\left\{d_{\vec x'}\vec v'(\vec
x',t')\right\}\cdot A(t)=d_{\vec x}\left(A(t)\cdot \vec v(\vec x,t)+
\frac{d A}{dt}(t)\cdot\vec x+
\frac{d\vec z}{dt}(t)\right)=A(t)\cdot \left\{d_{\vec x}\vec v(\vec
x,t)\right\}+\frac{d A}{dt}(t),
\end{equation}
and thus
\begin{equation}
\label{NoIn5redbstrhhhhjkhhkzz}\left(\lim\limits_{|\vec
x'|\to+\infty}\left\{d_{\vec x'}\vec v'(\vec
x',t')\right\}\right)\cdot A(t)=A(t)\cdot \left(\lim\limits_{|\vec
x|\to+\infty}\left\{d_{\vec x}\vec v(\vec
x,t)\right\}\right)+\frac{d A}{dt}(t).
\end{equation}
Therefore, by \er{NoIn5redbstrhhhhjkhhkzz}, in both systems $(*)$
and $(**)$ we have simultaneously
\begin{equation} \label{jljjjjkhhhhkhhjhjkjkjkljojkjkjkhj1hhjhkkjjljlsg}
\lim\limits_{|\vec x|\to+\infty}\left(d_{\vec x}\vec v(\vec
x,t)\right)=0\quad\quad\quad\forall t\in [t_0,+\infty)\,,
\end{equation}
and
\begin{equation} \label{jljjjjkhhhhkhhjhjkjkjkljojkjkjkhj1hhjhkkjjljljykukl}
\lim\limits_{|\vec x'|\to+\infty}\left(d_{\vec x'}\vec v'(\vec
x',t')\right)=0\quad\quad\quad\forall t'\in [t_0,+\infty)\,,
\end{equation}
if and only if
\begin{equation}
\label{jljjjjkhhhhkhhjhjkjkjkljojkjkjkhj1hhjhkkjjljljykuklljkj}
\frac{d A}{dt}(t)=0\quad\quad\quad\forall t\,,
\end{equation}
or equivalently
\begin{equation} \label{jljjjjkhhhhkhhjhjkjkjkljojkjkjkhj1hhjhkkjjljljykuklljkbmbm}
A(t)=B\quad\quad\quad\forall t\,,
\end{equation}
where $B\in SO(3)$ is a constant rotation, so that
\begin{equation}\label{noninchgravortbstrjkjkjggg33jklkhjhjhuiuiljjljjkklrhkjhjhjhhh}
\begin{cases}
\vec x'=B\cdot\vec x+\vec z(t),\\
t'=t.
\end{cases}
\end{equation}
and
\begin{equation}
\label{NoIn5redbstrhhhhkljlj}\vec v'(\vec x',t')=B\cdot \vec v(\vec
x,t)+ \frac{d\vec z}{dt}(t).
\end{equation}
In particular, by Definitions \ref{gjgjffjjjjjjjjjjjjjjjjjjjjjjhhr}
and \ref{jhjkhhjhjhgjg;kkljkljkkjlhlhhlh} $\mathcal{V}$ is indeed a
non-rotational rigid motion group. Moreover, by
\er{NoIn5redbstrhhhhkljlj} we deduce
\begin{equation}
\label{NoIn5redbstrhhhhkljljkhhjhjhjhgj}\lim\limits_{|\vec
x'|\to+\infty}\vec v'(\vec x',t')=B\cdot \left(\lim\limits_{|\vec
x|\to+\infty}\vec v(\vec x,t)\right)+ \frac{d\vec z}{dt}(t).
\end{equation}
and thus, in both systems $(*)$ and $(**)$ the limits
$\lim\limits_{|\vec x'|\to+\infty}\vec v'(\vec x',t')$ and
$\lim\limits_{|\vec x|\to+\infty}\vec v(\vec x,t)$ are independent
on time simultaneously if and only if $\vec w:=\frac{d\vec
z}{dt}(t)$ is a constant (independent of time) vector. In that case
there exists constant vectors $\vec c\in \mathbb{R}^3$ and $\vec
w\in \mathbb{R}^3$ so that
\begin{equation}
\label{NoIn5redbstrhhhhkljljkhhjhjhjhgjjkjkjkjh} \vec z(t)=\vec
c+\vec w t\quad\quad\forall t.
\end{equation}
Inserting \er{NoIn5redbstrhhhhkljljkhhjhjhjhgjjkjkjkjh} into
\er{noninchgravortbstrjkjkjggg33jklkhjhjhuiuiljjljjkklrhkjhjhjhhh}
gives
\begin{equation}\label{noninchgravortbstrjkjkjggg33jklkhjhjhuiuiljjljjkklrhkjhjhjhhhjkhgj}
\begin{cases}
\vec x'=B\cdot\vec x+\vec
c+\vec w t,\\
t'=t.
\end{cases}
\end{equation}
In particular, by Definitions \ref{gjgjffjjjjjjjjjjjjjjjjjjjjjjhh}
and \ref{jhjkhhjhjhgjg;kkljkljkkjlhlhhlh} $\mathcal{U}$ is indeed an
extended Galilean group.
\end{proof}

\begin{definition}\label{gjgjffjjjjjjjjjjjjjjjjjjjjjjhh2}
Given an \underline{extended Galilean group} $\mathcal{U}$ (see
Definition \ref{gjgjffjjjjjjjjjjjjjjjjjjjjjjhh}), we say that the
\underline{scalar} field $Z:=Z(\vec
x,t):\R^3\times[t_0,+\infty)\to\R$, defined in every cartesian
system inside the set $\mathcal{U}$, is a \underline{speed-like
field generator} on $\mathcal{U}$ if, under every change of two
coordinate systems $(*)$ and $(**)$ inside $\mathcal{U}$, given by
\er{noninchgravortbstrjkjkjggg33jklkhjhjhuiuiljjljjkkl}, the
quantity $Z$ transforms as:
\begin{equation}\label{noninchgravortbstrjgghguittu1intmmjhhjhihgggy}
Z'(\vec x',t')=Z(\vec x,t)+\vec w\cdot\left(B\cdot\vec x
\right)+\frac{1}{2}|\vec w|^2t+C_{*}^{**},
\end{equation}
where $C_{*}^{**}\in\mathbb{R}$ is an unspecified constant,
independent of the space and time variables $(\vec x,t)$, which is,
however, allowed to depend on both coordinate systems $(*)$ and
$(**)$.
\end{definition}
\begin{lemma}\label{ygjhgjgjgjhgjtu}
Let $\vec v:=\vec v(\vec x,t):\R^3\times[t_0,+\infty)\to\R^3$ be a
\underline{speed-like} vector field. Furthermore, let $\mathcal{U}$
be an \underline{extended Galilean group} and let $Z:=Z(\vec
x,t):\R^3\times[t_0,+\infty)\to\R$, be a \underline{speed-like field
generator} on $\mathcal{U}$, such that there exists cartesian
coordinate system $(*)$ inside $\mathcal{U}$, where we have
\begin{equation}\label{vfyutuyfffhhgfhgfh33iy33hhjgghgj}
\vec v(\vec x,t)=\nabla_{\vec x}Z(\vec x,t)\quad\quad\forall(\vec
x,t) \in\R^3\times[t_0,+\infty)\quad\quad\quad\text{in the given
particular system $(*)$}\,.
\end{equation}
Then, in \underline{every} cartesian coordinate system $(**)$ inside
the set $\mathcal{U}$ we also have
\begin{equation}\label{vfyutuyfffhhgfhgfh33iy33hhjgghjhgh}
\vec v'(\vec x',t')=\nabla_{\vec x'}Z'(\vec
x',t')\quad\quad\forall(\vec x',t')
\in\R^3\times[t_0,+\infty)\quad\quad\quad\text{in the system
$(**)$}\,.
\end{equation}
\end{lemma}
\begin{proof}
By \er{noninchgravortbstrjgghguittu1intmmjhhjhihgggy} and
\er{noninchgravortbstrjkjkjggg33jklkhjhjhuiuiljjljjkkl} we have
\begin{multline}\label{vfyutuyfffhhgfhgfh33iy33hhjgghjhghgjhggh}
\nabla_{\vec x'}Z'(\vec
x',t')=\left(B^{-1}\right)^T\cdot\nabla_{\vec x}\left(Z(\vec
x,t)+\vec w\cdot\left(B\cdot\vec x
\right)+\frac{1}{2}|\vec w|^2t+C_{*}^{**}\right)=\\
B\cdot\nabla_{\vec x}Z(\vec x,t)+B\cdot(B^T\cdot\vec
w)=B\cdot\nabla_{\vec x}Z(\vec x,t)+ \vec w\quad\quad\forall(\vec
x,t) \in\R^3\times[t_0,+\infty).
\end{multline}
On the other hand, the law of transformation of a speed-like vector
field in \er{NoIn5redbstr} of Definition \ref{bggghghgj} with
$A(t):=B$ and $\vec z(t):=(\vec c+\vec w t)$ in the case of
\er{noninchgravortbstrjkjkjggg33jklkhjhjhuiuiljjljjkkl} reeds as
\begin{equation}\label{vfyutuyfffhhgfhgfh33iy33hhjgghjhghgjgghgg}
\vec v'(\vec x',t')=B\cdot \vec v(\vec x,t)+\vec
w\quad\quad\forall(\vec x,t) \in\R^3\times[t_0,+\infty)\,.
\end{equation}
Thus, by \er{vfyutuyfffhhgfhgfh33iy33hhjgghjhghgjhggh} and
\er{vfyutuyfffhhgfhgfh33iy33hhjgghjhghgjgghgg} together with
\er{vfyutuyfffhhgfhgfh33iy33hhjgghgj} we finally deduce
\er{vfyutuyfffhhgfhgfh33iy33hhjgghjhgh}.
\end{proof}
\begin{lemma}\label{yghgjtgyrtrgjgjgjgjgjgjgjgjgjgjgjgjgj33hh}
Let $\mathcal{U}$ be an \underline{extended Galilean group} and let
$(\{0\})$ be some fixed inertial or non-inertial cartesian
coordinate system inside $\mathcal{U}$. Next consider a scalar field
$\gamma:=\gamma(\vec x,t):\R^3\times[t_0,+\infty)\to\R$ associated
\underline{only} with the chosen coordinate system $(\{0\})$. Then
there exist a
\underline{speed-like field generator} on $\mathcal{U}$, denoted by
$Z_{\gamma}:\R^3\times[t_0,+\infty)\to\R$ (and associated with any
inertial or non-inertial cartesian coordinate system inside
$\mathcal{U}$), such that in the \underline{particular} system
$(\{0\})$ we have:
\begin{equation}\label{vfyutuyfffhhgfhgfh33iy33hh}
Z_{\gamma}(\vec x,t)=\gamma(\vec x,t)\quad\quad\forall(\vec x,t)
\in\R^3\times[t_0,+\infty)\,.
\end{equation}
\end{lemma}
\begin{proof}
Similarly to
\er{noninchgravortbstrjkjkjggg33jklkhjhjhuiuiljjljjkkl}, given an
arbitrary inertial or non-inertial cartesian coordinate system
inside $\mathcal{U}$, denoted by the sign $(\{1\})$, the change of
coordinate system $(\{0\})$ to the system $(\{1\})$ can be written
as:
\begin{equation}\label{noninchgravortbstrjkjkbbq133hh}
\begin{cases}
\vec x^{(1)}=B_1\cdot\vec x+\vec c_1+\vec w_1 t,\\
t^{(1)}=t,
\end{cases}
\end{equation}
where $B_1\in SO(3)$ is a constant rotation and $\vec c_1,\vec
w_1\in \mathbb{R}^3$ are constant vectors. Similarly, given another
arbitrary inertial or non-inertial cartesian coordinate system
inside $\mathcal{U}$, denoted by the sign $(\{2\})$, the change of
coordinate system $(\{0\})$ to the system $(\{2\})$ can be written
as:
\begin{equation}\label{noninchgravortbstrjkjkbbq233hh}
\begin{cases}
\vec x^{(2)}=B_2\cdot\vec x+\vec c_2+\vec w_2 t,\\
t^{(2)}=t,
\end{cases}
\end{equation}
where $B_2\in SO(3)$ is another constant rotation and $\vec c_2,\vec
w_2\in \mathbb{R}^3$ are other constant vectors. On the other hand,
the change coordinate system $(\{1\})$ to the system $(\{2\})$ can
be written as:
\begin{equation}\label{noninchgravortbstrjkjkjggg33hh}
\begin{cases}
\vec x^{(2)}=B\cdot\vec x^{(1)}+\vec c+\vec w t^{(1)},\\
t^{(2)}=t^{(1)},
\end{cases}
\end{equation}
where $B\in SO(3)$ is a constant rotation and $\vec c,\vec w\in
\mathbb{R}^3$ are constant vectors. In particular, inserting
\er{noninchgravortbstrjkjkbbq133hh} into
\er{noninchgravortbstrjkjkjggg33hh} gives
\begin{equation}\label{noninchgravortbstrjkjkjggg33hkjh33hh}
\vec x^{(2)}=B\cdot\left(B_1\cdot\vec x+\vec c_1+\vec w_1
t\right)+\vec c+\vec  w t=\left(B\cdot B_1\right)\cdot\vec
x+\left(B\cdot\vec c_1+\vec c\right)+\left(B\cdot\vec w_1 +\vec w
\right)t.
\end{equation}
Thus, comparing \eqref{noninchgravortbstrjkjkjggg33hkjh33hh} with
\er{noninchgravortbstrjkjkbbq233hh} we easily deduce:
\begin{equation}\label{noninchgravortbstrjkjkjggg33hkjh33jklhhk33hh}
B_2=\left(B\cdot B_1\right),\quad\vec c_2=\left(B\cdot\vec c_1+\vec
c\right)\quad\quad\text{and}\quad\quad\vec w_2=\left(B\cdot\vec w_1
+\vec w \right).
\end{equation}

Next, for the system $(\{1\})$ define
\begin{equation}
\label{NoIn5redbstrhkhhk133hh}Z^{(1)}_{\gamma}(\vec
x^{(1)},t^{(1)})=\gamma(\vec x,t)+\vec w_1\cdot\left(B_1\cdot\vec x
\right)+\frac{1}{2}\left|\vec w_1\right|^2t.
\end{equation}
Correspondingly, for the system $(\{2\})$ we have
\begin{equation}
\label{NoIn5redbstrhkhhk233hh}Z^{(2)}_{\gamma}(\vec
x^{(2)},t^{(2)})=\gamma(\vec x,t)+\vec w_2\cdot\left(B_2\cdot\vec x
\right)+\frac{1}{2}\left|\vec w_2\right|^2t.
\end{equation}
Thus, in the \underline{particular} system $(\{0\})$ we have:
\begin{equation}\label{vfyutuyfffhhgfhgfh33iy33hkjhjh33hh}
Z_{\gamma}(\vec x,t)=\gamma(\vec x,t)\quad\quad\forall(\vec x,t)
\in\R^3\times[t_0,+\infty)\,.
\end{equation}
Thus in order to complete the proof we need to show that the scalar
field, defined by \er{NoIn5redbstrhkhhk133hh} and
\er{NoIn5redbstrhkhhk233hh} is a speed-like field generator. Indeed,
by \er{NoIn5redbstrhkhhk133hh} and \er{NoIn5redbstrhkhhk233hh}
together with \er{noninchgravortbstrjkjkjggg33hkjh33jklhhk33hh} and
\er{noninchgravortbstrjkjkbbq133hh} we deduce:
\begin{multline}
\label{NoIn5redbstrhkhhk233jhgjgj33hh}Z^{(2)}_{\gamma}(\vec
x^{(2)},t^{(2)})=\gamma(\vec x,t)+\vec w_2\cdot\left(B_2\cdot\vec x
\right)+\frac{1}{2}\left|\vec w_2\right|^2t \\=\gamma(\vec
x,t)+\left(B\cdot\vec w_1 +\vec w \right)\cdot\left(B\cdot
B_1\cdot\vec x
\right)+\frac{1}{2}\left|B\cdot\vec w_1 +\vec w\right|^2t=\\
\gamma(\vec x,t)+\left(B\cdot\vec w_1\right)\cdot\left(B\cdot \left(
B_1\cdot\vec x
\right)
\right) +\vec w\cdot\left(B\cdot \left( B_1\cdot\vec x
\right)
\right) +\frac{1}{2}\left|B\cdot\vec w_1\right|^2t
+\frac{1}{2}\left|\vec w\right|^2t+\left(\vec w\cdot\left(B\cdot\vec
w_1\right)\right)t\\
= \gamma(\vec x,t)+\vec w_1\cdot\left( B_1\cdot\vec x
\right) +\vec w\cdot\left(B\cdot \left(\vec x^{(1)}-\vec c_1
\right)\right) +\frac{1}{2}\left|\vec w_1\right|^2t
+\frac{1}{2}\left|\vec w\right|^2t
\\=
\left\{\gamma(\vec x,t)+\vec w_1\cdot\left( B_1\cdot\vec x
\right)+\frac{1}{2}\left|\vec w_1\right|^2t\right\} +\vec
w\cdot\left(B\cdot \vec x^{(1)}
\right) +\frac{1}{2}\left|\vec w\right|^2t-\vec
w\cdot\left(B\cdot\vec c_1\right)
\\=
Z^{(1)}_{\gamma}(\vec x^{(1)},t^{(1)})+\vec w\cdot\left(B\cdot \vec
x^{(1)}
\right) +\frac{1}{2}\left|\vec w\right|^2t^{(1)}-\vec
w\cdot\left(B\cdot\vec c_1\right),
\end{multline}
where we also used the identity $(B\cdot \vec x)\cdot(B\cdot \vec
y)=\vec x\cdot\vec y$, which is valid for any $B\in SO(3)$ and any
$\vec x,\vec y\in \mathbb{R}^3$. Thus, comparing
\er{NoIn5redbstrhkhhk233jhgjgj33hh} with
\er{noninchgravortbstrjgghguittu1intmmjhhjhihgggy} in Definition
\ref{gjgjffjjjjjjjjjjjjjjjjjjjjjjhh2}, we deduce that $Z_{\gamma}$
is a \underline{speed-like field generator} on $\mathcal{U}$. This
completes the proof.
\end{proof}

\begin{lemma}\label{ygjhgjgjgjhgjtuiuijk}
Let $Z:=Z(\vec x,t):\R^3\times[t_0,+\infty)\to\R$, be a
\underline{speed-like field generator} on an extended Galilean group
$\mathcal{U}$ and let $(*)$ and $(**)$ be two cartesian coordinate
systems  inside $\mathcal{U}$, so that the change of coordinates
from the system $(*)$ to the system $(**)$ is given by
\begin{equation}\label{noninchgravortbstrjkjkjggg33jklkhjhjhuiuiljjljjkkljkjk}
\begin{cases}
\vec x'=B\cdot\vec x+\vec c+\vec w t,\\
t'=t,
\end{cases}
\end{equation}
with constant rotation $B\in SO(3)$ and constant vectors $\vec
c,\vec w\in \mathbb{R}^3$, and correspondingly we have
\begin{equation}\label{noninchgravortbstrjgghguittu1intmmjhhjhihgggyhhhh}
Z'(\vec x',t')=Z(\vec x,t)+\vec w\cdot\left(B\cdot\vec x
\right)+\frac{1}{2}|\vec w|^2t+C_{*}^{**},
\end{equation}
where $C_{*}^{**}\in\mathbb{R}$ is an unspecified constant. Then we
have
\begin{equation}\label{noninchgravortbstrjgghguittu1intmmjhhjhihgggyhhhhjjjih}
\frac{\partial Z'}{\partial t'}(\vec
x',t')+\frac{1}{2}\left|\nabla_{\vec x'}Z'(\vec
x',t')\right|^2=\frac{\partial Z}{\partial t}(\vec
x,t)+\frac{1}{2}\left|\nabla_{\vec x}Z(\vec
x,t)\right|^2\quad\quad\forall(\vec x,t)
\in\R^3\times[t_0,+\infty)\,.
\end{equation}
\end{lemma}
\begin{proof}
As before, in \er{vfyutuyfffhhgfhgfh33iy33hhjgghjhghgjhggh} we have
\begin{multline}\label{vfyutuyfffhhgfhgfh33iy33hhjgghjhghgjhgghijhhh}
\nabla_{\vec x'}Z'(\vec
x',t')=\left(B^{-1}\right)^T\cdot\nabla_{\vec x}\left(Z(\vec
x,t)+\vec w\cdot\left(B\cdot\vec x
\right)+\frac{1}{2}|\vec w|^2t+C_{*}^{**}\right)=\\
B\cdot\nabla_{\vec x}Z(\vec x,t)+B\cdot(B^T\cdot\vec
w)=B\cdot\nabla_{\vec x}Z(\vec x,t)+ \vec w\quad\quad\forall(\vec
x,t) \in\R^3\times[t_0,+\infty).
\end{multline}
On the other hand, by the Chain rule we have
\begin{multline}\label{vfyutuyfffhhgfhgfh33iy33hhjgghjhghgjhgghijhhhjhihhhhi}
\frac{\partial Z'}{\partial t'}+\vec w\cdot\nabla_{\vec
x'}Z'=\frac{\partial Z'}{\partial t}=\frac{\partial}{\partial
t}\left(Z(\vec x,t)+\vec w\cdot\left(B\cdot\vec x
\right)+\frac{1}{2}|\vec w|^2t+C_{*}^{**}\right)=\frac{\partial
Z}{\partial t}+\frac{1}{2}|\vec w|^2.
\end{multline}
Therefore, by \er{vfyutuyfffhhgfhgfh33iy33hhjgghjhghgjhgghijhhh} and
\er{vfyutuyfffhhgfhgfh33iy33hhjgghjhghgjhgghijhhhjhihhhhi} we have:
\begin{multline}\label{vfyutuyfffhhgfhgfh33iy33hhjgghjhghgjhgghijhhhjhihhhhijhggj}
\frac{\partial Z'}{\partial t'}+\frac{1}{2}\left|\nabla_{\vec
x'}Z'\right|^2=\frac{\partial Z}{\partial t}+\frac{1}{2}|\vec
w|^2-\vec w\cdot\nabla_{\vec x'}Z'+\frac{1}{2}\left|\nabla_{\vec
x'}Z'\right|^2=\frac{\partial Z}{\partial
t}+\frac{1}{2}\left|\nabla_{\vec x'}Z'-\vec
w\right|^2\\=\frac{\partial Z}{\partial
t}+\frac{1}{2}\left|B\cdot\nabla_{\vec x}Z\right|^2=\frac{\partial
Z}{\partial t}+\frac{1}{2}\left|\nabla_{\vec x}Z\right|^2.
\end{multline}
This completes the proof.
\end{proof}

\begin{proposition}\label{jgjjjjjjjjjjjjjjjjjjjjjjjjjjjjjjjjjkhhkh}
\begin{itemize}
\item[{{\bf (i)}}]
Consider the moving continuum medium with velocity field $\vec
u(\vec x,t)$ and let
$\mathcal{S}:=\mathcal{S}(t)\subset\mathbb{R}^3$ be a
two-dimensional surface oriented by the unit normal $\vec n:=\vec
n(\vec x,t)$ and moving together with the given medium. Then, given
a vector field $\vec f(\vec x,t)$ we have:
\begin{equation}\label{MaxVacFullPPNffGGvbccgcu8uuuuukjjkhhhihjhjkhhkmgjhjhhkhjhhjjhjhjjh11ljkjkljkljjjkjkljkjk}
\frac{d}{dt}\left(\iint\vec f\cdot\vec n\,
d\mathcal{S}(t)\right)=\iint\left(\frac{\partial \vec f}{\partial
t}- curl_{\vec x}\left(\vec u\times \vec f\right)+\left({div}_{\vec
x}\vec f\right)\vec u\right)\cdot\vec n\, d\mathcal{S}(t).
\end{equation}

\item[{{\bf (ii)}}]
Consider the moving continuum medium with velocity field $\vec
u(\vec x,t)$ and let $\gamma:=\gamma(t)\subset\mathbb{R}^3$ be a
one-dimensional curve oriented by the unit tangent vector $\vec
t:=\vec t(\vec x,t)$ and moving together with the given medium.
Then, given a vector field $\vec f(\vec x,t)$ we have:
\begin{equation}\label{MaxVacFullPPNffGGvbccgcu8uuuuukjjkhhhihjhjkhhkmgjhjhhkhjhhjjhjhjjh11ljkjkljkljjjkjkljkjk11}
\frac{d}{dt}\left(\int\vec f\cdot\vec t\,
d\gamma(t)\right)=\int\left(\frac{\partial \vec f}{\partial t}- \vec
u\times curl_{\vec x}\vec f+\nabla_{\vec x}\left(\vec u\cdot\vec
f\right)\right)\cdot\vec t\, d\gamma(t).
\end{equation}

\item[{{\bf (iii)}}]
Consider the moving continuum medium with velocity field $\vec
u(\vec x,t)$ and let $\Omega:=\Omega(t)\subset\mathbb{R}^3$ be a
three-dimensional domain moving together with the given medium.
Then, given a scalar field $\psi(\vec x,t)$ we have:
\begin{equation}\label{MaxVacFullPPNffGGvbccgcu8uuuuukjjkhhhihjhjkhhkmgjhjhhkhjhhjjhjhjjh11ljkjkljkljjjkjkljkjk22}
\frac{d}{dt}\left(\iiint\psi\,
d\Omega(t)\right)=\iiint\left(\frac{\partial \psi}{\partial
t}+{div}_{\vec x}\left\{\psi \vec u\right\}\right)\, d\Omega(t).
\end{equation}
\end{itemize}
\end{proposition}
\begin{proof}
Assume that $\mathcal{S}(t)$ is given by the following:
\begin{equation}\label{MaxVacFullPPNffGGvbccgcu8uuuuukjjkhhhihjhjkhhkmgjhjhhkhjhhjjhjhjjh11ljkjkljkljjjkjkljkjkjkkhhkhhjh}
\mathcal{S}(t):=\left\{\vec x\in\mathbb{R}^3\,:\;\vec x=\vec
h(u,v,t),\;(u,v)\in \mathcal{W}\right\},
\end{equation}
where $\mathcal{W}\subset\mathbb{R}^2$ is some flat domain (in the
general case $\mathcal{S}(t)$ is just the union of several parts,
where every part is as in
\er{MaxVacFullPPNffGGvbccgcu8uuuuukjjkhhhihjhjkhhkmgjhjhhkhjhhjjhjhjjh11ljkjkljkljjjkjkljkjkjkkhhkhhjh}).
Then, as it is well known from the calculus, we have
\begin{equation}\label{MaxVacFullPPNffGGvbccgcu8uuuuukjjkhhhihjhjkhhkmgjhjhhkhjhhjjhjhjjh11ljkjkljkljjjkjkljkjkuhjhgj11}
\iint\vec f\cdot\vec n\,
d\mathcal{S}(t)=\iint\limits_\mathcal{W}\vec f\left(\vec
h(u,v,t),t\right)\cdot\left(\frac{\partial\vec h}{\partial
u}(u,v,t)\times \frac{\partial\vec h}{\partial
v}(u,v,t)\right)\,dudv.
\end{equation}
Therefore, differentiating
\er{MaxVacFullPPNffGGvbccgcu8uuuuukjjkhhhihjhjkhhkmgjhjhhkhjhhjjhjhjjh11ljkjkljkljjjkjkljkjkuhjhgj11},
by chain rule we obtain
\begin{multline}\label{MaxVacFullPPNffGGvbccgcu8uuuuukjjkhhhihjhjkhhkmgjhjhhkhjhhjjhjhjjh11ljkjkljkljjjkjkljkjkuhjhgjgdghf}
\frac{d}{dt}\left(\iint\vec f\cdot\vec n\,
d\mathcal{S}(t)\right)=\iint\limits_\mathcal{W}\frac{d}{dt}\left(\vec
f\left(\vec h(u,v,t),t\right)\right)\cdot\left(\frac{\partial\vec
h}{\partial u}(u,v,t)\times \frac{\partial\vec h}{\partial
v}(u,v,t)\right)\,dudv\\+\iint\limits_\mathcal{W}\vec f\left(\vec
h(u,v,t),t\right)\cdot\left(\frac{\partial^2\vec h}{\partial
u\partial t}(u,v,t)\times \frac{\partial\vec h}{\partial
v}(u,v,t)+\frac{\partial\vec h}{\partial u}(u,v,t)\times
\frac{\partial^2\vec h}{\partial v\partial t}(u,v,t)\right)\,dudv=\\
\iint\limits_\mathcal{W}\left(\frac{\partial\vec f}{\partial
t}\left(\vec h(u,v,t),t\right)+d_{\vec x}\vec f\left(\vec
h(u,v,t),t\right)\cdot\frac{\partial\vec h}{\partial
t}(u,v,t)\right)\cdot\left(\frac{\partial\vec h}{\partial
u}(u,v,t)\times \frac{\partial\vec h}{\partial
v}(u,v,t)\right)\,dudv\\+\iint\limits_\mathcal{W}\vec f\left(\vec
h(u,v,t),t\right)\cdot\left(\frac{\partial^2\vec h}{\partial
u\partial t}(u,v,t)\times \frac{\partial\vec h}{\partial
v}(u,v,t)-\frac{\partial^2\vec h}{\partial v\partial
t}(u,v,t)\times\frac{\partial\vec h}{\partial
u}(u,v,t)\right)\,dudv.
\end{multline}
On the other hand, since the surface $\mathcal{S}(t)$ moves together
with the medium with velocity field $\vec u(\vec x,t)$ we obviously
have:
\begin{equation}\label{MaxVacFullPPNffGGvbccgcu8uuuuukjjkhhhihjhjkhhkmgjhjhhkhjhhjjhjhjjh11ljkjkljkljjjkjkljkjkuhjhgj11jgjghghhhkhhkhkhkhk}
\frac{\partial\vec h}{\partial t}(u,v,t)=\vec u\left(\vec
h(u,v,t),t\right).
\end{equation}
Moreover, differentiating
\er{MaxVacFullPPNffGGvbccgcu8uuuuukjjkhhhihjhjkhhkmgjhjhhkhjhhjjhjhjjh11ljkjkljkljjjkjkljkjkuhjhgj11jgjghghhhkhhkhkhkhk}
by $u$ and $v$ and using again the chain rule gives:
\begin{equation}\label{MaxVacFullPPNffGGvbccgcu8uuuuukjjkhhhihjhjkhhkmgjhjhhkhjhhjjhjhjjh11ljkjkljkljjjkjkljkjkuhjhgj11jgjghghhhkhhkhkhkhkgjgggb}
\begin{cases} \frac{\partial^2\vec h}{\partial u\partial t}(u,v,t)=d_{\vec x}\vec u\left(\vec
h(u,v,t),t\right)\cdot \frac{\partial\vec h}{\partial
u}(u,v,t),\\
\frac{\partial^2\vec h}{\partial v\partial t}(u,v,t)=d_{\vec x}\vec
u\left(\vec h(u,v,t),t\right)\cdot \frac{\partial\vec h}{\partial
v}(u,v,t).
\end{cases}
\end{equation}
Thus inserting
\er{MaxVacFullPPNffGGvbccgcu8uuuuukjjkhhhihjhjkhhkmgjhjhhkhjhhjjhjhjjh11ljkjkljkljjjkjkljkjkuhjhgj11jgjghghhhkhhkhkhkhk}
and
\er{MaxVacFullPPNffGGvbccgcu8uuuuukjjkhhhihjhjkhhkmgjhjhhkhjhhjjhjhjjh11ljkjkljkljjjkjkljkjkuhjhgj11jgjghghhhkhhkhkhkhkgjgggb}
into
\er{MaxVacFullPPNffGGvbccgcu8uuuuukjjkhhhihjhjkhhkmgjhjhhkhjhhjjhjhjjh11ljkjkljkljjjkjkljkjkuhjhgjgdghf}
gives:
\begin{multline}\label{MaxVacFullPPNffGGvbccgcu8uuuuukjjkhhhihjhjkhhkmgjhjhhkhjhhjjhjhjjh11ljkjkljkljjjkjkljkjkuhjhgjgdghfkkjkkhhkhkhkhj}
\frac{d}{dt}\left(\iint\vec f\cdot\vec n\, d\mathcal{S}(t)\right)=\\
\iint\limits_\mathcal{W}\left(\frac{\partial\vec f}{\partial
t}\left(\vec h(u,v,t),t\right)+d_{\vec x}\vec f\left(\vec
h(u,v,t),t\right)\cdot\vec u\left(\vec
h(u,v,t),t\right)\right)\cdot\left(\frac{\partial\vec h}{\partial
u}(u,v,t)\times \frac{\partial\vec h}{\partial
v}(u,v,t)\right)\,dudv\\+\iint\limits_\mathcal{W}\vec f\left(\vec
h(u,v,t),t\right)\cdot\left(\left(d_{\vec x}\vec u\left(\vec
h(u,v,t),t\right)\cdot \frac{\partial\vec h}{\partial
u}(u,v,t)\right)\times \frac{\partial\vec h}{\partial
v}(u,v,t)\right)\,dudv
\\-\iint\limits_\mathcal{W}\vec f\left(\vec
h(u,v,t),t\right)\cdot\left(\left(d_{\vec x}\vec u\left(\vec
h(u,v,t),t\right)\cdot \frac{\partial\vec h}{\partial
v}(u,v,t)\right)\times\frac{\partial\vec h}{\partial
u}(u,v,t)\right)\,dudv.
\end{multline}
On the other hand, using \er{apfrm10} we rewrite
\er{MaxVacFullPPNffGGvbccgcu8uuuuukjjkhhhihjhjkhhkmgjhjhhkhjhhjjhjhjjh11ljkjkljkljjjkjkljkjkuhjhgjgdghfkkjkkhhkhkhkhj}
as:
\begin{multline}\label{MaxVacFullPPNffGGvbccgcu8uuuuukjjkhhhihjhjkhhkmgjhjhhkhjhhjjhjhjjh11ljkjkljkljjjkjkljkjkuhjhgjgdghfkkjkkhhkhkhkhjhkjhjhgjg}
\frac{d}{dt}\left(\iint\vec f\cdot\vec n\, d\mathcal{S}(t)\right)=\\
\iint\limits_\mathcal{W}\left(\frac{\partial\vec f}{\partial
t}\left(\vec h(u,v,t),t\right)+d_{\vec x}\vec f\left(\vec
h(u,v,t),t\right)\cdot\vec u\left(\vec
h(u,v,t),t\right)\right)\cdot\left(\frac{\partial\vec h}{\partial
u}(u,v,t)\times \frac{\partial\vec h}{\partial
v}(u,v,t)\right)\,dudv\\+\iint\limits_\mathcal{W}\left(tr\left\{d_{\vec
x}\vec u\left(\vec h(u,v,t),t\right)\right\}\right)\,\vec
f\left(\vec h(u,v,t),t\right)\cdot \left(\frac{\partial\vec
h}{\partial u}(u,v,t)\times \frac{\partial\vec h}{\partial
v}(u,v,t)\right)\,dudv
\\-\iint\limits_\mathcal{W}\vec f\left(\vec
h(u,v,t),t\right)\cdot\left(\left\{d_{\vec x}\vec u\left(\vec
h(u,v,t),t\right)\right\}^T\cdot \left(\frac{\partial\vec
h}{\partial u}(u,v,t)\times\frac{\partial\vec h}{\partial
v}(u,v,t)\right)\right)\,dudv
=\\
\iint\limits_\mathcal{W}\left(\frac{\partial\vec f}{\partial
t}\left(\vec h(u,v,t),t\right)+d_{\vec x}\vec f\left(\vec
h(u,v,t),t\right)\cdot\vec u\left(\vec
h(u,v,t),t\right)\right)\cdot\left(\frac{\partial\vec h}{\partial
u}(u,v,t)\times \frac{\partial\vec h}{\partial
v}(u,v,t)\right)\,dudv\\+\iint\limits_\mathcal{W}\left(div_{\vec
x}\vec u\left(\vec h(u,v,t),t\right)\right)\,\vec f\left(\vec
h(u,v,t),t\right)\cdot \left(\frac{\partial\vec h}{\partial
u}(u,v,t)\times \frac{\partial\vec h}{\partial
v}(u,v,t)\right)\,dudv
\\-\iint\limits_\mathcal{W}\left(d_{\vec x}\vec u\left(\vec
h(u,v,t),t\right)\cdot\vec f\left(\vec
h(u,v,t),t\right)\right)\cdot\left(\frac{\partial\vec h}{\partial
u}(u,v,t)\times\frac{\partial\vec h}{\partial v}(u,v,t)\right)\,dudv
.
\end{multline}
Thus, by \er{apfrm6} we rewrite
\er{MaxVacFullPPNffGGvbccgcu8uuuuukjjkhhhihjhjkhhkmgjhjhhkhjhhjjhjhjjh11ljkjkljkljjjkjkljkjkuhjhgjgdghfkkjkkhhkhkhkhjhkjhjhgjg}
as:
\begin{multline}\label{MaxVacFullPPNffGGvbccgcu8uuuuukjjkhhhihjhjkhhkmgjhjhhkhjhhjjhjhjjh11ljkjkljkljjjkjkljkjkuhjhgjgdghfkkjkkhhkhkhkhjhkjhjhgjggghhfhkh}
\frac{d}{dt}\left(\iint\vec f\cdot\vec n\, d\mathcal{S}(t)\right) =
\iint\limits_\mathcal{W} \left(\frac{\partial \vec f}{\partial
t}\left(\vec h(u,v,t),t\right)\right)\cdot\left(\frac{\partial\vec
h}{\partial u}(u,v,t)\times\frac{\partial\vec h}{\partial
v}(u,v,t)\right)
\,dudv\\
-\iint\limits_\mathcal{W} \left(curl_{\vec x}\left(\vec u\left(\vec
h(u,v,t),t\right)\times \vec f\left(\vec
h(u,v,t),t\right)\right)\right)\cdot\left(\frac{\partial\vec
h}{\partial u}(u,v,t)\times\frac{\partial\vec h}{\partial
v}(u,v,t)\right) \,dudv \\+\iint\limits_\mathcal{W}
\left(\left({div}_{\vec x}\vec f\left(\vec
h(u,v,t),t\right)\right)\vec u\left(\vec
h(u,v,t),t\right)\right)\cdot\left(\frac{\partial\vec h}{\partial
u}(u,v,t)\times\frac{\partial\vec h}{\partial v}(u,v,t)\right)
\,dudv  .
\end{multline}
Therefore, using the reverse direction of
\er{MaxVacFullPPNffGGvbccgcu8uuuuukjjkhhhihjhjkhhkmgjhjhhkhjhhjjhjhjjh11ljkjkljkljjjkjkljkjkuhjhgj11}
in the right hand side of
\er{MaxVacFullPPNffGGvbccgcu8uuuuukjjkhhhihjhjkhhkmgjhjhhkhjhhjjhjhjjh11ljkjkljkljjjkjkljkjkuhjhgjgdghfkkjkkhhkhkhkhjhkjhjhgjggghhfhkh}
gives
\begin{multline}\label{MaxVacFullPPNffGGvbccgcu8uuuuukjjkhhhihjhjkhhkmgjhjhhkhjhhjjhjhjjh11ljkjkljkljjjkjkljkjkuhjhgjgdghfkkjkkhhkhkhkhjhkjhjhgjggghhfhkhkjkhj}
\frac{d}{dt}\left(\iint\vec f\cdot\vec n\, d\mathcal{S}(t)\right) =
\iint\frac{\partial \vec f}{\partial t}\cdot\vec n\, d\mathcal{S}(t)
-\iint\left(curl_{\vec x}\left(\vec u\times \vec
f\right)\right)\cdot\vec n\, d\mathcal{S}(t) +\iint\left({div}_{\vec
x}\vec f\right)\vec u\cdot\vec n\, d\mathcal{S}(t),
\end{multline}
and thus
\er{MaxVacFullPPNffGGvbccgcu8uuuuukjjkhhhihjhjkhhkmgjhjhhkhjhhjjhjhjjh11ljkjkljkljjjkjkljkjk}
follows. So we proved part {\bf (i)} of the Proposition.

Next assume that $\gamma(t)$ is given by the following:
\begin{equation}\label{MaxVacFullPPNffGGvbccgcu8uuuuukjjkhhhihjhjkhhkmgjhjhhkhjhhjjhjhjjh11ljkjkljkljjjkjkljkjkjkkhhkhhjh11}
\gamma(t):=\left\{\vec x\in\mathbb{R}^3\,:\;\vec x=\vec
h(s,t),\;s\in[0,1]\right\},
\end{equation}
(in the general case $\gamma(t)$ is just the union of several parts,
where every part is as in
\er{MaxVacFullPPNffGGvbccgcu8uuuuukjjkhhhihjhjkhhkmgjhjhhkhjhhjjhjhjjh11ljkjkljkljjjkjkljkjkjkkhhkhhjh11}).
Then, as it is well known from the calculus, we have
\begin{equation}\label{MaxVacFullPPNffGGvbccgcu8uuuuukjjkhhhihjhjkhhkmgjhjhhkhjhhjjhjhjjh11ljkjkljkljjjkjkljkjkuhjhgj1111}
\int\vec f\cdot\vec t\, d\gamma(t)=\int_0^1\vec f\left(\vec
h(s,t),t\right)\cdot\frac{\partial\vec h}{\partial s}(s,t)\,ds.
\end{equation}
Therefore, differentiating
\er{MaxVacFullPPNffGGvbccgcu8uuuuukjjkhhhihjhjkhhkmgjhjhhkhjhhjjhjhjjh11ljkjkljkljjjkjkljkjkuhjhgj1111},
by chain rule we obtain
\begin{multline}\label{MaxVacFullPPNffGGvbccgcu8uuuuukjjkhhhihjhjkhhkmgjhjhhkhjhhjjhjhjjh11ljkjkljkljjjkjkljkjkuhjhgjgdghf1111}
\frac{d}{dt}\left(\int\vec f\cdot\vec t\,
d\gamma(t)\right)=\int_0^1\frac{d}{dt}\left(\vec f\left(\vec
h(s,t),t\right)\right)\cdot\frac{\partial\vec h}{\partial
s}(s,t)\,ds+\int_0^1\vec f\left(\vec
h(s,t),t\right)\cdot\frac{\partial^2\vec h}{\partial s\partial
t}(s,t)\,ds=\\
\int_0^1\left(\frac{\partial\vec f}{\partial t}\left(\vec
h(s,t),t\right)+d_{\vec x}\vec f\left(\vec
h(s,t),t\right)\cdot\frac{\partial\vec h}{\partial
t}(s,t)\right)\cdot\frac{\partial\vec h}{\partial
s}(s,t)\,ds+\int_0^1\vec f\left(\vec
h(s,t),t\right)\cdot\frac{\partial^2\vec h}{\partial s\partial
t}(s,t)\,ds.
\end{multline}
On the other hand, since the curve $\gamma(t)$ moves together with
the medium with velocity field $\vec u(\vec x,t)$ we obviously have:
\begin{equation}\label{MaxVacFullPPNffGGvbccgcu8uuuuukjjkhhhihjhjkhhkmgjhjhhkhjhhjjhjhjjh11ljkjkljkljjjkjkljkjkuhjhgj11jgjghghhhkhhkhkhkhk11}
\frac{\partial\vec h}{\partial t}(s,t)=\vec u\left(\vec
h(s,t),t\right).
\end{equation}
Moreover, differentiating
\er{MaxVacFullPPNffGGvbccgcu8uuuuukjjkhhhihjhjkhhkmgjhjhhkhjhhjjhjhjjh11ljkjkljkljjjkjkljkjkuhjhgj11jgjghghhhkhhkhkhkhk11}
by $s$ and using again the chain rule gives:
\begin{equation}\label{MaxVacFullPPNffGGvbccgcu8uuuuukjjkhhhihjhjkhhkmgjhjhhkhjhhjjhjhjjh11ljkjkljkljjjkjkljkjkuhjhgj11jgjghghhhkhhkhkhkhkgjgggb11}
\frac{\partial^2\vec h}{\partial s\partial t}(s,t)=d_{\vec x}\vec
u\left(\vec h(s,t),t\right)\cdot \frac{\partial\vec h}{\partial
s}(s,t).
\end{equation}
Thus inserting
\er{MaxVacFullPPNffGGvbccgcu8uuuuukjjkhhhihjhjkhhkmgjhjhhkhjhhjjhjhjjh11ljkjkljkljjjkjkljkjkuhjhgj11jgjghghhhkhhkhkhkhk11}
and
\er{MaxVacFullPPNffGGvbccgcu8uuuuukjjkhhhihjhjkhhkmgjhjhhkhjhhjjhjhjjh11ljkjkljkljjjkjkljkjkuhjhgj11jgjghghhhkhhkhkhkhkgjgggb11}
into
\er{MaxVacFullPPNffGGvbccgcu8uuuuukjjkhhhihjhjkhhkmgjhjhhkhjhhjjhjhjjh11ljkjkljkljjjkjkljkjkuhjhgjgdghf1111}
gives:
\begin{multline}\label{MaxVacFullPPNffGGvbccgcu8uuuuukjjkhhhihjhjkhhkmgjhjhhkhjhhjjhjhjjh11ljkjkljkljjjkjkljkjkuhjhgjgdghfkkjkkhhkhkhkhj11}
\frac{d}{dt}\left(\int\vec f\cdot\vec t\, d\gamma(t)\right)=
\int_0^1\left(\frac{\partial\vec f}{\partial t}\left(\vec
h(s,t),t\right)+d_{\vec x}\vec f\left(\vec h(s,t),t\right)\cdot\vec
u\left(\vec h(s,t),t\right)\right)\cdot\frac{\partial\vec
h}{\partial s}(s,t)\,ds\\+\int_0^1\vec f\left(\vec
h(s,t),t\right)\cdot\left(d_{\vec x}\vec u\left(\vec
h(s,t),t\right)\cdot \frac{\partial\vec h}{\partial
s}(s,t)\right)\,ds\\= \int_0^1\left(\frac{\partial\vec f}{\partial
t}\left(\vec h(s,t),t\right)+d_{\vec x}\vec f\left(\vec
h(s,t),t\right)\cdot\vec u\left(\vec
h(s,t),t\right)\right)\cdot\frac{\partial\vec h}{\partial
s}(s,t)\,ds\\+\int_0^1\left(\left\{d_{\vec x}\vec u\left(\vec
h(s,t),t\right)\right\}^T\cdot \vec f\left(\vec
h(s,t),t\right)\right)\cdot\frac{\partial\vec h}{\partial
s}(s,t)\,ds.
\end{multline}
On the other hand, using \er{apfrm9} we rewrite
\er{MaxVacFullPPNffGGvbccgcu8uuuuukjjkhhhihjhjkhhkmgjhjhhkhjhhjjhjhjjh11ljkjkljkljjjkjkljkjkuhjhgjgdghfkkjkkhhkhkhkhj11}
as:
\begin{multline}\label{MaxVacFullPPNffGGvbccgcu8uuuuukjjkhhhihjhjkhhkmgjhjhhkhjhhjjhjhjjh11ljkjkljkljjjkjkljkjkuhjhgjgdghfkkjkkhhkhkhkhjhkjhjhgjggghhfhkh11}
\frac{d}{dt}\left(\int\vec f\cdot\vec t\, d\gamma(t)\right)=
\int_0^1\left(\frac{\partial\vec f}{\partial t}\left(\vec
h(s,t),t\right)-\vec u\left(\vec h(s,t),t\right)\times curl_{\vec
x}\vec f\left(\vec h(s,t),t\right)\right)\cdot\frac{\partial\vec
h}{\partial s}(s,t)\,ds\\+\int_0^1\left(\left\{d_{\vec x}\vec
f\left(\vec h(s,t),t\right)\right\}^T\cdot \vec u\left(\vec
h(s,t),t\right)+\left\{d_{\vec x}\vec u\left(\vec
h(s,t),t\right)\right\}^T\cdot \vec f\left(\vec
h(s,t),t\right)\right)\cdot\frac{\partial\vec h}{\partial
s}(s,t)\,ds\\
= \int_0^1\left(\frac{\partial\vec f}{\partial t}\left(\vec
h(s,t),t\right)-\vec u\left(\vec h(s,t),t\right)\times curl_{\vec
x}\vec f\left(\vec h(s,t),t\right)\right)\cdot\frac{\partial\vec
h}{\partial s}(s,t)\,ds\\+\int_0^1\nabla_{\vec x}\left(\vec
u\left(\vec h(s,t),t\right)\cdot \vec f\left(\vec
h(s,t),t\right)\right)\cdot\frac{\partial\vec h}{\partial
s}(s,t)\,ds.
\end{multline}
Therefore, using the reverse direction of
\er{MaxVacFullPPNffGGvbccgcu8uuuuukjjkhhhihjhjkhhkmgjhjhhkhjhhjjhjhjjh11ljkjkljkljjjkjkljkjkuhjhgj1111}
in the right hand side of
\er{MaxVacFullPPNffGGvbccgcu8uuuuukjjkhhhihjhjkhhkmgjhjhhkhjhhjjhjhjjh11ljkjkljkljjjkjkljkjkuhjhgjgdghfkkjkkhhkhkhkhjhkjhjhgjggghhfhkh11}
gives
\begin{equation}\label{MaxVacFullPPNffGGvbccgcu8uuuuukjjkhhhihjhjkhhkmgjhjhhkhjhhjjhjhjjh11ljkjkljkljjjkjkljkjkuhjhgjgdghfkkjkkhhkhkhkhjhkjhjhgjggghhfhkhkjkhj11}
\frac{d}{dt}\left(\int\vec f\cdot\vec t\, d\gamma(t)\right) =
\int\left(\frac{\partial\vec f}{\partial t}-\vec u\times curl_{\vec
x}\vec f\right)\cdot\vec t\, d\gamma(t)+\int\nabla_{\vec
x}\left(\vec u\cdot \vec f\right)\cdot\vec t\, d\gamma(t) ,
\end{equation}
and thus
\er{MaxVacFullPPNffGGvbccgcu8uuuuukjjkhhhihjhjkhhkmgjhjhhkhjhhjjhjhjjh11ljkjkljkljjjkjkljkjk11}
follows. So we proved part {\bf(ii)} of the Proposition.

Next assume that $\Omega(t)$ is given by the following:
\begin{equation}\label{MaxVacFullPPNffGGvbccgcu8uuuuukjjkhhhihjhjkhhkmgjhjhhkhjhhjjhjhjjh11ljkjkljkljjjkjkljkjkjkkhhkhhjh22}
\Omega(t):=\left\{\vec x\in\mathbb{R}^3\,:\;\vec x=\vec h(\vec
y,t),\;\vec y\in G\right\},
\end{equation}
where $G\subset\mathbb{R}^3$ is some three-dimensional domain. Then,
as it is well known from the rule of the change of variables of
integration from the calculus, we have
\begin{equation}\label{MaxVacFullPPNffGGvbccgcu8uuuuukjjkhhhihjhjkhhkmgjhjhhkhjhhjjhjhjjh11ljkjkljkljjjkjkljkjkuhjhgj1122}
\iiint\psi\, d\Omega(t)=\iiint\limits_G\vec \psi\left(\vec
h(u,v,t),t\right)\left|det\left\{d_{\vec y}\vec h(\vec
y,t)\right\}\right|\,d\vec y.
\end{equation}

Therefore, differentiating
\er{MaxVacFullPPNffGGvbccgcu8uuuuukjjkhhhihjhjkhhkmgjhjhhkhjhhjjhjhjjh11ljkjkljkljjjkjkljkjkuhjhgj1122},
by chain rule we obtain
\begin{multline}\label{MaxVacFullPPNffGGvbccgcu8uuuuukjjkhhhihjhjkhhkmgjhjhhkhjhhjjhjhjjh11ljkjkljkljjjkjkljkjkuhjhgjgdghf1122}
\frac{d}{dt}\left(\iiint\psi\,
d\Omega(t)\right)=\iiint\limits_G\frac{d}{dt}\left(\vec
\psi\left(\vec h(\vec y,t),t\right)\right)\left|det\left\{d_{\vec
y}\vec h(\vec y,t)\right\}\right|\,d\vec y\\+\iiint\limits_G\vec
\psi\left(\vec h(\vec y,t),t\right)\frac{\partial}{
\partial
t}\left\{\left|det\left\{d_{\vec y}\vec h(\vec
y,t)\right\}\right|\right\}\,d\vec y=\\
\iiint\limits_G\left(\frac{\partial\vec \psi}{\partial t}\left(\vec
h(\vec y,t),t\right)+\nabla_{\vec x}\vec \psi\left(\vec h(\vec
y,t),t\right)\cdot\frac{\partial\vec h}{\partial t}(\vec
y,t)\right)\left|det\left\{d_{\vec y}\vec h(\vec
y,t)\right\}\right|\,d\vec y\\+\iiint\limits_G\vec \psi\left(\vec
h(\vec y,t),t\right)\frac{\partial}{
\partial
t}\left\{\left|det\left\{d_{\vec y}\vec h(\vec
y,t)\right\}\right|\right\}\,d\vec y.
\end{multline}
On the other hand, since the domain $\Omega(t)$ moves together with
the medium with velocity field $\vec u(\vec x,t)$ we obviously have:
\begin{equation}\label{MaxVacFullPPNffGGvbccgcu8uuuuukjjkhhhihjhjkhhkmgjhjhhkhjhhjjhjhjjh11ljkjkljkljjjkjkljkjkuhjhgj11jgjghghhhkhhkhkhkhk22}
\frac{\partial\vec h}{\partial t}(\vec y,t)=\vec u\left(\vec h(\vec
y,t),t\right).
\end{equation}
Moreover, differentiating
\er{MaxVacFullPPNffGGvbccgcu8uuuuukjjkhhhihjhjkhhkmgjhjhhkhjhhjjhjhjjh11ljkjkljkljjjkjkljkjkuhjhgj11jgjghghhhkhhkhkhkhk22}
by $\vec y$ and using again the chain rule gives:
\begin{equation}\label{MaxVacFullPPNffGGvbccgcu8uuuuukjjkhhhihjhjkhhkmgjhjhhkhjhhjjhjhjjh11ljkjkljkljjjkjkljkjkuhjhgj11jgjghghhhkhhkhkhkhkgjgggbkk22}
\frac{\partial}{\partial t}\left\{d_{\vec y}\vec h(\vec
y,t)\right\}=d_{\vec x}\vec u\left(\vec h(\vec y,t),t\right)\cdot
\left\{d_{\vec y}\vec h(\vec y,t)\right\}.
\end{equation}
In particular, using the Liouville Theorem in the theory of linear
ordinary differential systems by
\er{MaxVacFullPPNffGGvbccgcu8uuuuukjjkhhhihjhjkhhkmgjhjhhkhjhhjjhjhjjh11ljkjkljkljjjkjkljkjkuhjhgj11jgjghghhhkhhkhkhkhkgjgggbkk22}
we deduce that $det\left\{d_{\vec y}\vec h(\vec y,t)\right\}$
satisfies
\begin{equation}\label{MaxVacFullPPNffGGvbccgcu8uuuuukjjkhhhihjhjkhhkmgjhjhhkhjhhjjhjhjjh11ljkjkljkljjjkjkljkjkuhjhgj11jgjghghhhkhhkhkhkhkgjgggbzzz22}
\frac{\partial}{\partial t}\left(det\left\{d_{\vec y}\vec h(\vec
y,t)\right\}\right)=\left(tr\left\{d_{\vec x}\vec u\left(\vec h(\vec
y,t),t\right)\right\}\right) \left(det\left\{d_{\vec y}\vec h(\vec
y,t)\right\}\right)=\left(div_{\vec x}\vec u\left(\vec h(\vec
y,t),t\right)\right) \left(det\left\{d_{\vec y}\vec h(\vec
y,t)\right\}\right),
\end{equation}
and thus,
\begin{equation}\label{MaxVacFullPPNffGGvbccgcu8uuuuukjjkhhhihjhjkhhkmgjhjhhkhjhhjjhjhjjh11ljkjkljkljjjkjkljkjkuhjhgj11jgjghghhhkhhkhkhkhkgjgggb22}
\frac{\partial}{\partial t}\left(\left|det\left\{d_{\vec y}\vec
h(\vec y,t)\right\}\right|\right)=\left(div_{\vec x}\vec u\left(\vec
h(\vec y,t),t\right)\right) \left|det\left\{d_{\vec y}\vec h(\vec
y,t)\right\}\right|.
\end{equation}
Thus inserting
\er{MaxVacFullPPNffGGvbccgcu8uuuuukjjkhhhihjhjkhhkmgjhjhhkhjhhjjhjhjjh11ljkjkljkljjjkjkljkjkuhjhgj11jgjghghhhkhhkhkhkhk22}
and
\er{MaxVacFullPPNffGGvbccgcu8uuuuukjjkhhhihjhjkhhkmgjhjhhkhjhhjjhjhjjh11ljkjkljkljjjkjkljkjkuhjhgj11jgjghghhhkhhkhkhkhkgjgggb22}
into
\er{MaxVacFullPPNffGGvbccgcu8uuuuukjjkhhhihjhjkhhkmgjhjhhkhjhhjjhjhjjh11ljkjkljkljjjkjkljkjkuhjhgjgdghf1122}
gives:
\begin{multline}\label{MaxVacFullPPNffGGvbccgcu8uuuuukjjkhhhihjhjkhhkmgjhjhhkhjhhjjhjhjjh11ljkjkljkljjjkjkljkjkuhjhgjgdghfkkjkkhhkhkhkhj22}
\frac{d}{dt}\left(\iiint\psi\, d\Omega(t)\right)=
\iiint\limits_G\left(\frac{\partial\vec \psi}{\partial t}\left(\vec
h(\vec y,t),t\right)+\vec u\left(\vec h(\vec
y,t),t\right)\cdot\nabla_{\vec x}\vec\psi\left(\vec h(\vec
y,t),t\right)\right)\left|det\left\{d_{\vec y}\vec h(\vec
y,t)\right\}\right|\,d\vec y\\+\iiint\limits_G\vec \psi\left(\vec
h(\vec y,t),t\right)\left(div_{\vec x}\vec u\left(\vec h(\vec
y,t),t\right)\right) \left|det\left\{d_{\vec y}\vec h(\vec
y,t)\right\}\right|\,d\vec y.
\end{multline}
Therefore, using the reverse direction of
\er{MaxVacFullPPNffGGvbccgcu8uuuuukjjkhhhihjhjkhhkmgjhjhhkhjhhjjhjhjjh11ljkjkljkljjjkjkljkjkuhjhgj1122}
in the right hand side of
\er{MaxVacFullPPNffGGvbccgcu8uuuuukjjkhhhihjhjkhhkmgjhjhhkhjhhjjhjhjjh11ljkjkljkljjjkjkljkjkuhjhgjgdghfkkjkkhhkhkhkhj22}
gives
\begin{multline}\label{MaxVacFullPPNffGGvbccgcu8uuuuukjjkhhhihjhjkhhkmgjhjhhkhjhhjjhjhjjh11ljkjkljkljjjkjkljkjkuhjhgjgdghfkkjkkhhkhkhkhj22khhhkhhhj22}
\frac{d}{dt}\left(\iiint\psi\, d\Omega(t)\right)=
\iiint\left(\frac{\partial\vec \psi}{\partial t}+\vec
u\cdot\nabla_{\vec x}\psi+(div_{\vec x}\vec u)\psi\right)\,
d\Omega(t)\\= \iiint\left(\frac{\partial\vec \psi}{\partial
t}+div_{\vec x}\left\{\psi\vec u\right\}\right)\, d\Omega(t),
\end{multline}
and thus
\er{MaxVacFullPPNffGGvbccgcu8uuuuukjjkhhhihjhjkhhkmgjhjhhkhjhhjjhjhjjh11ljkjkljkljjjkjkljkjk22}
follows. So we proved part {\bf(iii)} of the Proposition.
%
%
%
\end{proof}

\begin{lemma}\label{khhjgjgggh}
The matrix valued field $T:=T(\vec
x,t):\R^3\times[t_0,+\infty)\to\R^{3\times 3}$ is a proper matrix
field if and only if, under every change of coordinate system given
by \er{noninchgravortbstr}, this field transforms as:
\begin{equation}\label{uguyytfddddgghjjghjjj11zz}
T'(\vec x',t')=\left(A(t)\otimes A(t)\right)\cdot T(\vec x,t),
\end{equation}
where the sign $\otimes$ in \er{uguyytfddddgghjjghjjj11zz} means the
tensor product of the matrices, i.e. for given two linear operators
(matrices) $A\in\mathbb{R}^{3\times 3}$ and $B\in
\mathbb{R}^{3\times 3}$
their tensor product
$A\otimes B$ is a linear operator from $\mathbb{R}^{3\times 3}$ to
$\mathbb{R}^{3\times 3}$, defined by the identity:
\begin{equation}\label{fyfgjguyiuiHHHHHHkjjjjlzz}
(A\otimes B)\cdot(\vec a\otimes \vec b)=(A\cdot \vec a)\otimes
(B\cdot \vec b)\quad\quad\forall \vec a\in\mathbb{R}^{3},\;\forall
\vec b\in\mathbb{R}^{3}.
\end{equation}
\end{lemma}
\begin{proof}
Since we identify a vector $\vec c\in \mathbb{R}^{3}$ with the
corresponding $3\times 1$ matrix, clearly we have
\begin{equation}\label{gjgghhikkyy}
\vec a\otimes \vec b=\vec a\cdot\vec b^T\quad\quad\forall \vec
a\in\mathbb{R}^{3},\;\forall \vec b\in\mathbb{R}^{3}.
\end{equation}
Thus, inserting \er{gjgghhikkyy} into \er{fyfgjguyiuiHHHHHHkjjjjlzz}
gives:
\begin{multline}\label{fyfgjguyiuiHHHHHHkjjjjlzzgjgghghzzaazz}
(A\otimes A)\cdot(\vec a\otimes \vec b)=(A\cdot \vec a)\otimes
(A\cdot \vec b)=(A\cdot \vec a)\cdot(A\cdot \vec b)^T=(A\cdot \vec
a)\cdot(\vec b^T\cdot A^T)=A\cdot(\vec a\cdot\vec b^T)\cdot
A^T\\=A\cdot(\vec a\otimes\vec b)\cdot A^T\quad\quad\forall \vec
a\in\mathbb{R}^{3},\;\forall \vec b\in\mathbb{R}^{3}.
\end{multline}
Then, by the linearity, \er{fyfgjguyiuiHHHHHHkjjjjlzzgjgghghzzaazz}
gives:
\begin{equation}\label{fyfgjguyiuiHHHHHHkjjjjlzzgjgghghzz}
(A\otimes A)\cdot G=A\cdot G\cdot A^T\quad\quad\forall\,
G\in\mathbb{R}^{3\times 3}.
\end{equation}
Thus, by \er{fyfgjguyiuiHHHHHHkjjjjlzzgjgghghzz} we deduce that
\er{uguyytfddddgghjjghjjj11zz} is equivalent to
\er{uguyytfddddgghjjghjjj}.
\end{proof}

\section{Gravity revised}\label{gugyu} Consider the classical space-time
where the change of some inertial coordinate system $(*)$ to another
inertial coordinate system $(**)$ is given (up to equivalence) by
the Galilean Transformation:
\begin{equation}\label{noninchgravortbstrjgghguittu1}
\begin{cases}
\vec x'=\vec x+\vec wt,\\
t'=t,
\end{cases}
\end{equation}
and the change of some non-inertial cartesian coordinate system
$(*)$ to another cartesian coordinate system $(**)$ is of the form:
\begin{equation}\label{noninchgravortbstrjgghguittu2}
\begin{cases}
\vec x'=A(t)\cdot\vec x+\vec z(t),\\
t'=t,
\end{cases}
\end{equation}
where $A(t)\in SO(3)$ is a rotation, i.e. $A(t)\in \R^{3\times 3}$,
$det\, A(t)>0$ and $A(t)\cdot A^T(t)=I$, where $A^T$ is the
transpose of the matrix $A$.

 Similarly to the General Theory of Relativity, we assume that
the most general laws of Classical Mechanics should be invariant in
every non-inertial cartesian coordinate system, i.e. they preserve
their form under transformations of the form
\er{noninchgravortbstrjgghguittu2}. Moreover, again as in the
General Theory of Relativity, we assume that the fictitious forces
(inertial forces) in non-inertial coordinate systems and the forces
of Newtonian gravitation have the same nature and represented by
some field in somewhat similar to the Electromagnetic field.

 We begin with some simple observation. Assume that we are away from
essential gravitational masses and strong electromagnetic fields.
Then consider two cartesian coordinate systems $(*)$ and $(**)$,
such that the system $(**)$ is inertial and the change of coordinate
system $(*)$ to coordinate system $(**)$ is given by
\er{noninchgravortbstrjgghguittu2}. Then the
fictitious-gravitational force in the system $(**)$ is trivial $\vec
F'_0=0$. On the other hand, since under the change of coordinate
system of the form \er{noninchgravortbstrjgghguittu2} the velocity
transforms as
\begin{equation}\label{noninchgravortbstrjgghguittu2gjg}
\vec u'=A(t)\cdot\vec u+\frac{dA}{dt}(t)\cdot\vec x+\frac{d\vec
z}{dt}(t)
\end{equation}
and the acceleration
transforms as
\begin{equation}\label{noninchgravortbstrjgghguittu2gjgghhjhg}
\vec a'=A(t)\cdot\vec a+2\frac{dA}{dt}(t)\cdot\vec
u+\frac{d^2A}{dt^2}(t)\cdot\vec x+\frac{d^2\vec z}{dt^2}(t)
\end{equation}
the fictitious-gravitational force in the system $(*)$, acting on
the particle with inertial mass $m$ and velocity $\vec u$, is given
by
\begin{equation}\label{noninchgravortbstrjgghguittu2gjgghhjhghjhjgg}
\vec F_0=m\left(-2A^T(t)\cdot\frac{dA}{dt}(t)\cdot\vec
u-A^T(t)\cdot\frac{d^2 A}{dt^2}(t)\cdot\vec
x-A^T(t)\cdot\frac{d^2\vec z}{dt^2}(t)\right).
\end{equation}
On the other hand, since $A(t)\cdot A^T(t)=I$ and thus
$A^T(t)\cdot\frac{dA}{dt}(t)+\frac{dA^T}{dt}(t)\cdot A(t)=0$, if we
define a vector field
\begin{equation}\label{noninchgravortbstrjgghguittu2gjgjhjhhklk}
\vec v(\vec x,t):=-A^T(t)\cdot\frac{dA}{dt}(t)\cdot\vec
x-A^T(t)\cdot\frac{d\vec z}{dt}(t),
\end{equation}
then we obviously have
\begin{equation}\label{noninchgravortbstrjgghguittu2gjgjhjhhklkhgffgfg}
\begin{cases}
d_{\vec x}\vec
v=-A^T(t)\cdot\frac{dA}{dt}(t)=\frac{dA^T}{dt}(t)\cdot
A(t)
\\
\left\{d_{\vec x}\vec v\right\}^T=-\frac{dA^T}{dt}(t)\cdot
A(t)=A^T(t)\cdot\frac{dA}{dt}(t)
\\
\frac{\partial\vec v}{\partial
t}=-A^T(t)\cdot\left(\frac{d^2A}{dt^2}(t)\cdot\vec x+\frac{d^2\vec
z}{dt^2}(t)\right)-\frac{dA^T}{dt}(t)\cdot\left(\frac{dA}{dt}(t)\cdot\vec
x+\frac{d\vec z}{dt}(t)\right)
\end{cases}
\end{equation}
Thus by \er{noninchgravortbstrjgghguittu2gjgjhjhhklk} and
\er{noninchgravortbstrjgghguittu2gjgjhjhhklkhgffgfg} we rewrite
\er{noninchgravortbstrjgghguittu2gjgghhjhghjhjgg} as
\begin{equation}\label{noninchgravortbstrjgghguittu2gjgghhjhghjhjgghgghghgh}
\vec F_0=m\left(-2A^T(t)\cdot\frac{dA}{dt}(t)\cdot\vec
u+\frac{\partial\vec v}{\partial t}-\frac{dA^T}{dt}(t)\cdot
A(t)\cdot \vec v\right).
\end{equation}
Then using \er{apfrm9} and
\er{noninchgravortbstrjgghguittu2gjgjhjhhklkhgffgfg} we finally
rewrite \er{noninchgravortbstrjgghguittu2gjgghhjhghjhjgghgghghgh} as
\begin{equation}\label{noninchgravortbstrjgghguittu2gjgghhjhghjhjgghgghghghtytyt}
\vec F_0=m\left(\frac{\partial\vec v}{\partial
t}+\frac{1}{2}\nabla_{\vec x}\left(|\vec v|^2\right)\right)+m\vec
u\times \left(-curl_{\vec x}\vec v\right).
\end{equation}

 Similarly assume that also in the general case of essential
gravitational masses there exists a vector field $\vec v(\vec x,t)$
such that in some inertial or non-inertial cartesian coordinate
system the fictitious-gravitational force is given by
\er{noninchgravortbstrjgghguittu2gjgghhjhghjhjgghgghghghtytyt}. Then
we call the vector field $\vec v$ the vectorial gravitational
potential. We see here the following analogy with Electrodynamics:
denoting
\begin{equation*}
\tilde {\vec E}:=\partial_{t}\vec v+\nabla_{\vec
x}\left(\frac{1}{2}|\vec v|^2\right)\quad\text{and}\quad \tilde
{\vec B}:=-c\, curl_{\vec x}\vec v,
\end{equation*}
we rewrite
\er{noninchgravortbstrjgghguittu2gjgghhjhghjhjgghgghghghtytyt} as
\begin{equation}\label{MaxVacFull1ninshtrgravortjhhjfhfhNewhjh}
\vec F_0=m\left(\tilde {\vec E}+\frac{1}{c}\vec u\times\tilde {\vec
B}\right),
\end{equation}
where
\begin{equation*}
curl_{\vec x}\tilde {\vec E}+\frac{1}{c}\frac{\partial}{\partial
t}\tilde {\vec B}=0\quad\text{and}\quad div_{\vec x}\tilde {\vec
B}=0.
\end{equation*}

Next using
\er{noninchgravortbstrjgghguittu2gjgghhjhghjhjgghgghghghtytyt} we
rewrite the Second Law of Newton as
\begin{equation}\label{noninchgravortbstrjgghguittu2gjgghhjhghjhjgghgghghghtytythvfghfgghjgg}
m\frac{d^2\vec x}{dt^2}=m\frac{d\vec u}{dt}=\vec F_0+\vec
F=m\left(\frac{\partial\vec v}{\partial t}(\vec
x,t)+\frac{1}{2}\nabla_{\vec x}\left(\left|\vec
v\right|^2\right)(\vec x,t)\right)+m\vec u\times \left(-curl_{\vec
x}\vec v(\vec x,t)\right)+\vec F,
\end{equation}
where $\vec x:=\vec x(t)$, $\vec u:=\vec u(t)=\frac{d\vec x}{dt}(t)$
and $m$ are the place, the velocity and the inertial mass of some
given particle at the moment of time $t$, $\vec v:=\vec v(\vec x,t)$
is the vectorial gravitational potential and $\vec F$ is the total
non-gravitational force, acting on the given particle.

Once we considered the Second Law of Newton in the form
\begin{equation}\label{MaxVacFull1ninshtrgravorthjhjlhhjPPN}
\frac{d\vec u}{dt}=-\vec u\times curl_{\vec x}\vec
v+\partial_{t}\vec v+\nabla_{\vec x}\left(\frac{1}{2}|\vec
v|^2\right)+\frac{1}{m}\vec F,
\end{equation}
we still need to prove that this law is invariant under the change
of inertial or non-inertial cartesian coordinate system and to
determine the law of transformation for the vectorial-gravitational
potential under the change of coordinate systems. As we will show
above this is indeed the case and moreover, the law of
transformation of the vectorial gravitational potential, under the
change of coordinate system, given by
\er{noninchgravortbstrjgghguittu2}, is:
$$\vec v'=A(t)\cdot \vec
v+\frac{dA}{dt}(t)\cdot\vec x+\frac{d\vec z}{dt}(t)$$ i.e. it is the
same as the transformation of a field of velocities. More precisely
we have the following:
\begin{proposition}\label{gjghghgghg}
Consider the change of some non-inertial cartesian coordinate system
$(*)$ to another cartesian coordinate system $(**)$ of the form:
\begin{equation}\label{noninchgravortPPNnnn}
\begin{cases}
\vec x'=A(t)\cdot\vec x+\vec z(t),\\
t'=t,
\end{cases}
\end{equation}
where $A(t)\in SO(3)$ is a rotation, i.e. $A(t)\in \R^{3\times 3}$,
$det\, A(t)>0$ and $A(t)\cdot A^T(t)=I$. Next, assume that in the
coordinate system $(**)$ we observe a validity of the Second Law of
Newton in the form:
\begin{equation}\label{MaxVacFull1ninshtrgravortPPN}
\frac{d\vec u'}{dt'}=-\vec u'\times curl_{\vec x'}\vec
v'+\partial_{t'}\vec v'+\nabla_{\vec x'}\left(\frac{1}{2}|\vec
v'|^2\right)+\frac{1}{m'}\vec F',
\end{equation}
where $\vec x':=\vec x'(t')$, $\vec u':=\vec u'(t')=\frac{d\vec
x'}{dt'}(t')$ and $m'$ are the place, the velocity and the mass of
some given particle at the moment of time $t'$, $\vec v':=\vec
v'(\vec x',t')$ is the vectorial gravitational potential and $\vec
F'$ is a total non-gravitational force, acting on the given particle
in the coordinate system $(**)$. Then in the coordinate system $(*)$
we observe a validity of the Second Law of Newton in the (same as
\er{MaxVacFull1ninshtrgravortPPN}) form:
\begin{equation}\label{MaxVacFull1ninshtrgravortjhhjPPNjffjf}
\frac{d\vec u}{dt}=-\vec u\times curl_{\vec x}\vec
v+\partial_{t}\vec v+\nabla_{\vec x}\left(\frac{1}{2}|\vec
v|^2\right)+\frac{1}{m}\vec F,
\end{equation}
where
\begin{align}
\label{NoIn5gravortPPN11}\vec v'=A(t)\cdot \vec
v+\frac{dA}{dt}(t)\cdot\vec x+\frac{d\vec z}{dt}(t)\\
\label{NoIn1gravortPPN11}\vec F'=A(t)\cdot\vec F,\\
\label{NoIn2gravortPPN11}m'=m,\\
\label{NoIn3gravortPPN11}\vec u'=A(t)\cdot \vec
u+\frac{dA}{dt}(t)\cdot\vec x+\frac{d\vec z}{dt}(t).
\end{align}
\end{proposition}
\begin{proof}
Using \er{apfrm9} we rewrite \er{MaxVacFull1ninshtrgravortPPN} as
\begin{equation}\label{MaxVacFull1ninshtrgravortghhghgPPN}
\frac{d\vec u'}{dt'}=-(\vec u'-\vec v')\times curl_{\vec x'}\vec
v'+\partial_{t'}\vec v'+d_{\vec x'}\vec v'\cdot\vec
v'+\frac{1}{m'}\vec F'.
\end{equation}
Next define the vector field $\vec v$ in the system $(*)$ in such a
way that it will be related to $\vec v'$ in the system $(**)$ due to
\er{NoIn5gravortPPN11}. I.e. $\vec v$ is given by $$\vec
v:=A^T(t)\cdot\left(\vec v'-\frac{dA}{dt}(t)\cdot\vec x-\frac{d\vec
z}{dt}(t)\right).$$ We are going to prove
\er{MaxVacFull1ninshtrgravortjhhjPPNjffjf} in the system $(*)$ using
the following relations between the physical characteristics in
coordinate systems $(*)$ and $(**)$:
\begin{align}
\label{NoIn1gravortPPN}\vec F'=A(t)\cdot\vec F,\\
\label{NoIn2gravortPPN}m'=m,\\
\label{NoIn3gravortPPN}\vec u'=A(t)\cdot \vec u+A'(t)\cdot\vec x+\vec w(t),\\
\label{NoIn5gravortPPN}\vec v'=A(t)\cdot \vec v+A'(t)\cdot\vec
x+\vec w(t),
\end{align}
where $\vec w(t):=\frac{d\vec z}{dt}(t)$ and
$A'(t)=\frac{dA}{dt}(t)$. Indeed, inserting these relations into
\er{MaxVacFull1ninshtrgravortghhghgPPN} we obtain:
\begin{multline}\label{MaxVacFull1ninshtrgravortghhghgjkggPPN}
\frac{d}{dt}\left(A(t)\cdot \vec u(t)+A'(t)\cdot\vec x(t)+\vec
w(t)\right)=-\left(A(t)\cdot(\vec u-\vec v)\right)\times curl_{\vec
x'}\left(A(t)\cdot \vec v+A'(t)\cdot\vec x+\vec
w(t)\right)\\+\partial_{t'}\left(A(t)\cdot \vec v+A'(t)\cdot\vec
x+\vec w(t)\right)+d_{\vec x'}\left(A(t)\cdot \vec v+A'(t)\cdot\vec
x+\vec w(t)\right)\cdot\left(A(t)\cdot \vec v+A'(t)\cdot\vec x+\vec
w(t)\right)\\+\frac{1}{m}A(t)\cdot\vec F.
\end{multline}
Next using the chain rule we deduce:
\begin{multline}
\label{vyguiuiujggghjjgggjhgghjortjggghkhjPPN}
\partial_{t'}\left(A(t)\cdot
\vec v+A'(t)\cdot\vec x+\vec w(t)\right)+d_{\vec x'}\left(A(t)\cdot
\vec v+A'(t)\cdot\vec x+\vec w(t)\right)\cdot\left(A'(t)\cdot\vec
x+\vec w(t)\right)=\\
\partial_{t}\left(A(t)\cdot \vec v+A'(t)\cdot\vec x+\vec
w(t)\right).
\end{multline}
Inserting it into \er{MaxVacFull1ninshtrgravortghhghgjkggPPN} we
deduce
\begin{multline}\label{MaxVacFull1ninshtrgravortghhghgjkgghklhPPN}
\frac{d}{dt}\left(A(t)\cdot \vec u(t)+A'(t)\cdot\vec x(t)+\vec
w(t)\right)=-\left(A(t)\cdot(\vec u-\vec v)\right)\times curl_{\vec
x'}\left(A(t)\cdot \vec v+A'(t)\cdot\vec x+\vec
w(t)\right)\\+\partial_{t}\left(A(t)\cdot \vec v+A'(t)\cdot\vec
x+\vec w(t)\right)+d_{\vec x}\left(A(t)\cdot \vec v+A'(t)\cdot\vec
x+\vec w(t)\right)\cdot \vec v+\frac{1}{m}A(t)\cdot\vec F.
\end{multline}
On the other hand, by \er{noninchgravortPPNnnn} and by Proposition
\ref{yghgjtgyrtrt}
we clearly have
\begin{equation}
\label{vyguiuiujggghjjgggorthjjhPPN} curl_{\vec x'}\left(
\left(A(t)\cdot\vec v+A'(t)\cdot\vec x+\vec
w(t)\right)\right)=A(t)\cdot curl_{\vec x}\left(\vec
v+A^{-1}(t)\cdot A'(t)\cdot\vec x+A^{-1}(t)\cdot \vec w(t)\right).
\end{equation}
Inserting it into \er{MaxVacFull1ninshtrgravortghhghgjkgghklhPPN} we
deduce:
\begin{multline}\label{MaxVacFull1ninshtrgravortghhghgjkgghklhjgkghglPPN}
\frac{d}{dt}\left(A(t)\cdot \vec u(t)+A'(t)\cdot\vec x(t)+\vec
w(t)\right)=\\-\left(A(t)\cdot(\vec u-\vec
v)\right)\times\left(A(t)\cdot curl_{\vec x}\left(\vec
v+A^{-1}(t)\cdot A'(t)\cdot\vec x+A^{-1}(t)\cdot
w(t)\right)\right)\\+\partial_{t}\left(A(t)\cdot \vec
v+A'(t)\cdot\vec x+\vec w(t)\right)+d_{\vec x}\left(A(t)\cdot \vec
v+A'(t)\cdot\vec x+\vec w(t)\right)\cdot \vec
v+\frac{1}{m}A(t)\cdot\vec F.
\end{multline}
Thus by \er{MaxVacFull1ninshtrgravortghhghgjkgghklhjgkghglPPN} and
\er{apfrm90} we have:
\begin{multline}\label{MaxVacFull1ninshtrgravortghhghgjkgghklhjgkghgPPN}
\frac{d}{dt}\left(A(t)\cdot \vec u(t)+A'(t)\cdot\vec x(t)+\vec
w(t)\right)=\\-A(t)\cdot \left((\vec u-\vec v)\times curl_{\vec
x}\left(\vec v+A^{-1}(t)\cdot A'(t)\cdot\vec x+A^{-1}(t)\cdot
w(t)\right)\right)\\+\partial_{t}\left(A(t)\cdot \vec
v+A'(t)\cdot\vec x+\vec w(t)\right)+d_{\vec x}\left(A(t)\cdot \vec
v+A'(t)\cdot\vec x+\vec w(t)\right)\cdot \vec
v+\frac{1}{m}A(t)\cdot\vec F.
\end{multline}
On the other hand clearly we have
$$\frac{d}{dt}\left(A(t)\cdot \vec u(t)+A'(t)\cdot\vec x(t)+\vec
w(t)\right)=A(t)\cdot\frac{d\vec u}{dt}+2A'(t)\cdot\vec
u+A''(t)\cdot \vec x(t)+\frac{d\vec w}{dt}(t).$$ Inserting it into
\er{MaxVacFull1ninshtrgravortghhghgjkgghklhjgkghgPPN} we deduce:
\begin{multline}\label{MaxVacFull1ninshtrgravortghhghgjkgghklhjgkghghjjkjPPN}
A(t)\cdot\frac{d\vec u}{dt}+2A'(t)\cdot\vec u+A''(t)\cdot \vec
x(t)+\frac{d\vec w}{dt}(t)=\\-A(t)\cdot \left((\vec u-\vec v)\times
curl_{\vec x}\left(\vec v+A^{-1}(t)\cdot A'(t)\cdot\vec
x+A^{-1}(t)\cdot w(t)\right)\right)\\+\partial_{t}\left(A(t)\cdot
\vec v+A'(t)\cdot\vec x+\vec w(t)\right)+d_{\vec x}\left(A(t)\cdot
\vec v+A'(t)\cdot\vec x+\vec
w(t)\right)\cdot \vec v+\frac{1}{m}A(t)\cdot\vec F\\
=-A(t)\cdot \left((\vec u-\vec v)\times curl_{\vec x}\vec
v\right)-A(t)\cdot \left((\vec u-\vec v)\times curl_{\vec
x}\left(A^{-1}(t)\cdot A'(t)\cdot\vec
x\right)\right)\\+A(t)\cdot\partial_t\vec v+2A'(t)\cdot \vec
v+A''(t)\cdot \vec x+\frac{d\vec w}{dt}(t)+A(t)\cdot d_{\vec x}\vec
v\cdot \vec v+\frac{1}{m}A(t)\cdot\vec F.
\end{multline}
We rewrite
\er{MaxVacFull1ninshtrgravortghhghgjkgghklhjgkghghjjkjPPN} as:
\begin{multline}\label{MaxVacFull1ninshtrgravortghhghgjkgghklhjgkghghjjkjhjkkggjkhjkPPN}
\frac{d\vec u}{dt}=- (\vec u-\vec v)\times curl_{\vec
x}\left(A^{-1}(t)\cdot A'(t)\cdot\vec x\right)-2A^{-1}(t)\cdot
A'(t)\cdot (\vec u-\vec v)\\-(\vec u-\vec v)\times curl_{\vec x}\vec
v +\partial_t\vec v+d_{\vec x}\vec v\cdot \vec v+\frac{1}{m}\vec F.
\end{multline}
Thus by \er{apfrm9} and
\er{MaxVacFull1ninshtrgravortghhghgjkgghklhjgkghghjjkjhjkkggjkhjkPPN}
we deduce:
\begin{multline}\label{MaxVacFull1ninshtrgravortghhghgjkgghklhjgkghghjjkjhjkkggjkhjkhjjhhPPN}
\frac{d\vec u}{dt}=d_{\vec x}\left(A^{-1}(t)\cdot A'(t)\cdot\vec
x\right)\cdot(\vec u-\vec v)-\left\{d_{\vec x}\left(A^{-1}(t)\cdot
A'(t)\cdot\vec x\right)\right\}^T\cdot(\vec u-\vec
v)-2A^{-1}(t)\cdot A'(t)\cdot (\vec u-\vec v)\\-(\vec u-\vec
v)\times curl_{\vec x}\vec v +\partial_t\vec v+d_{\vec x}\vec v\cdot
\vec v+\frac{1}{m}\vec F\\=\left(A^{-1}(t)\cdot
A'(t)\right)\cdot(\vec u-\vec v)-\left\{A^{-1}(t)\cdot
A'(t)\right\}^T\cdot(\vec u-\vec v)-2A^{-1}(t)\cdot A'(t)\cdot (\vec
u-\vec v)\\-(\vec u-\vec v)\times curl_{\vec x}\vec v
+\partial_t\vec v+d_{\vec x}\vec v\cdot \vec v+\frac{1}{m}\vec F.
\end{multline}
On the other hand the matrix $A^{-1}(t)\cdot A'(t)$ is antisymmetric
and thus
\begin{equation*}\left\{A^{-1}(t)\cdot
A'(t)\right\}^T=-\left(A^{-1}(t)\cdot A'(t)\right).
\end{equation*}
Inserting it into
\er{MaxVacFull1ninshtrgravortghhghgjkgghklhjgkghghjjkjhjkkggjkhjkhjjhhPPN}
we deduce:
\begin{equation}\label{MaxVacFull1ninshtrgravortghhghgjkgghklhjgkghghjjkjhjkkggjkhjkhjjhhfPPN}
\frac{d\vec u}{dt}=-(\vec u-\vec v)\times curl_{\vec x}\vec v
+\partial_t\vec v+d_{\vec x}\vec v\cdot \vec v+\frac{1}{m}\vec F.
\end{equation}
Thus again by \er{apfrm9} we finally rewrite
\er{MaxVacFull1ninshtrgravortghhghgjkgghklhjgkghghjjkjhjkkggjkhjkhjjhhfPPN}
as:
\begin{equation}\label{MaxVacFull1ninshtrgravortjhhjPPN}
\frac{d\vec u}{dt}=-\vec u\times curl_{\vec x}\vec
v+\partial_{t}\vec v+\nabla_{\vec x}\left(\frac{1}{2}|\vec
v|^2\right)+\frac{1}{m}\vec F.
\end{equation}
Therefore in the coordinate system $(*)$ we observe a validity of
Second Law of Newton in the same form as
\er{MaxVacFull1ninshtrgravortPPN}.
\end{proof}
\begin{remark}\label{hhghjggh}
Assume that in some inertial or non-inertial cartesian coordinate
system some particle with the place $\vec r(t)$ and velocity $\vec
u(t)=\frac{d\vec r}{dt}(t)$ moves in the gravitational field, and
all other forces, acting on the particle, except of the
gravitational forces are negligible. Then since, as before, by
\er{MaxVacFull1ninshtrgravorthjhjlhhjPPN} with $\vec F=0$ we have
\begin{multline}\label{MaxVacFull1ninshtrgravorthjhjlhhjPPNhggh}
\frac{d\vec u}{dt}(t)=-\left(\vec u(t)-\vec v\left(\vec
r(t),t\right)\right)\times curl_{\vec x}\vec v\left(\vec
r(t),t\right)+\partial_{t}\vec v\left(\vec r(t),t\right)+d_{\vec
x}\vec v\left(\vec r(t),t\right)\cdot\vec v\left(\vec
r(t),t\right)=\\ \partial_{t}\vec v\left(\vec r(t),t\right)+d_{\vec
x}\vec v\left(\vec r(t),t\right)\cdot\frac{d\vec r}{dt}(t) -d_{\vec
x}\vec v\left(\vec r(t),t\right)\cdot\left(\vec u(t)-\vec
v\left(\vec r(t),t\right)\right)-\left(\vec u(t)-\vec v\left(\vec
r(t),t\right)\right)\times curl_{\vec x}\vec v\left(\vec
r(t),t\right)\\= \frac{d}{dt}\left\{\vec v\left(\vec
r(t),t\right)\right\}-\left\{d_{\vec x}\vec v\left(\vec
r(t),t\right)\right\}^T\cdot\left(\vec u(t)-\vec v\left(\vec
r(t),t\right)\right),
\end{multline}
we deduce that the vectorial quantity $\left(\frac{d\vec
r}{dt}(t)-\vec v\left(\vec r(t),t\right)\right)=\left(\vec u(t)-\vec
v\left(\vec r(t),t\right)\right)$ satisfies the following first
order homogenous vectorial linear ordinary differential equation:
\begin{equation}\label{MaxVacFull1ninshtrgravorthjhjlhhjPPNhgghjjjjj}
\frac{d}{dt}\left\{u(t)-\vec v\left(\vec r(t),t\right)\right\}+
\left\{d_{\vec x}\vec v\left(\vec
r(t),t\right)\right\}^T\cdot\left\{\vec u(t)-\vec v\left(\vec
r(t),t\right)\right\}=0.
\end{equation}
In particular if for some instant of time $t_0$ we have
\begin{equation}\label{hjfgjffgffg}
\vec u(t_0):=\frac{d\vec r}{dt}(t_0)=\vec v\left(\vec
r(t_0),t_0\right)
\end{equation}
then by uniqueness theorem for ordinary differential equations,
\er{MaxVacFull1ninshtrgravorthjhjlhhjPPNhgghjjjjj} and
\er{hjfgjffgffg} together imply
\begin{equation}\label{hjfgjffgffgugh}
\vec u(t):=\frac{d\vec r}{dt}(t)=\vec v\left(\vec
r(t),t\right)\quad\quad\forall t,
\end{equation}
for every instant of time. I.e. if the velocity of the particle for
some initial instant of time coincides with the local vectorial
gravitational potential, then it will coincide with it at any
instant of time and the trajectory of motion will be tangent to the
direction of the local vectorial gravitational potential.
\end{remark}
\begin{remark}\label{jbkvjvjjccc}
Assume that some particle with the place $\vec r:=\vec r(t)$, the
velocity $\vec u:=\vec u(t)=\frac{d\vec r}{dt}(t)$ and the inertial
mass $m$ moves in the outer gravitational field with the vectorial
gravitational potential $\vec v:=\vec v(\vec x,t)$ in the absence of
non-gravitational forces. Then we can associate a Lagrangian with
\er{MaxVacFull1ninshtrgravorthjhjlhhjPPN}. Indeed, for this case we
define a Lagrangian:
\begin{equation}\label{vhfffngghkjgghfjjSYSPN11}
\mathcal{L}_0\left(\frac{d\vec r}{dt},\vec
r,t\right):=\frac{m}{2}\left|\frac{d\vec r}{dt}-\vec v(\vec
r,t)\right|^2.
\end{equation}
This Lagrangian is invariant under the change of non-inertial
cartesian coordinate systems, given by
\er{noninchgravortbstrjgghguittu2}. Moreover, we can easily deduce
(see subsection \ref{hggyugyuy} for more details) that a trajectory
$\vec r(t):[0,T]\to\mathbb{R}^3$ is a critical point of the
functional
\begin{equation}\label{btfffygtgyggyijhhkkSYSPN11}
I_0=\int_0^T \mathcal{L}_0\left(\frac{d\vec r}{dt}(t),\vec
r(t),t\right)dt
\end{equation}
if and only if it satisfies
\begin{equation}\label{vhfffngghkjgghggtghjgfhjoyuiyuyhiyyukukyihyuSYSPN11}
-m\frac{d^2\vec r}{dt^2}+m\left(\frac{\partial}{\partial t}\vec
v(\vec r,t)+\nabla_{\vec x}\left(\frac{1}{2}\left|\vec v(\vec
r,t)\right|^2\right)-\frac{d\vec r}{dt}\times curl_{\vec x}\vec
v\left(\vec r,t\right)\right)=0,
\end{equation}
consistently with \er{MaxVacFull1ninshtrgravorthjhjlhhjPPN} for the
case $\vec F=0$.
\end{remark}

Next, in order to fit the Second Law of Newton in the form
\er{MaxVacFull1ninshtrgravorthjhjlhhjPPN} with the classical Second
Law of Newton and the Newtonian Law of Gravity we consider that in
\underline{inertial} coordinate system $(*)$, at least in the first
approximation, we should have
\begin{equation}
\label{MaxVacFull1ninshtrgravortghhghgjkgghklhjgkghghjjkjhjkkggjkhjkhjjhhfhjhklkhkhjjklzzzyyyhjggjhgghhjhNWNWBWHWPPN222kgghjghjghj}
\begin{cases}
curl_{\vec x}\vec v= 0,\\
\frac{\partial\vec v}{\partial t}+\frac{1}{2}\nabla_{\vec
x}\left(|\vec v|^2\right)= -\nabla_{\vec x}\Phi,
\end{cases}
\end{equation}
where $\Phi$ is a scalar Newtonian gravitational potential which is
assumed to be a proper scalar, which satisfies
\begin{equation}
\label{MaxVacFull1ninshtrgravortghhghgjkgghklhjgkghghjjkjhjkkggjkhjkhjjhhfhjhklkhkhjjklzzzyyyhjggjhgghhjhNWNWNWBWHWPPN222}
\begin{cases}
\Delta_{\vec x}\Phi=4\pi GM\quad\quad\quad\forall \,(\vec x,t),\\
\lim\limits_{|\vec x|\to+\infty}\nabla_{\vec x}\Phi(\vec
x,t)=0\quad\quad\quad\forall t,
\end{cases}
\end{equation}
where $M$ is the gravitational mass density and $G$ is the
gravitational constant. Thus, since we require $curl_{\vec x}\vec v=
0$,
\er{MaxVacFull1ninshtrgravortghhghgjkgghklhjgkghghjjkjhjkkggjkhjkhjjhhfhjhklkhkhjjklzzzyyyhjggjhgghhjhNWNWBWHWPPN222kgghjghjghj}
is equivalent to:
\begin{equation}
\label{MaxVacFull1ninshtrgravortghhghgjkgghklhjgkghghjjkjhjkkggjkhjkhjjhhfhjhklkhkhjjklzzzyyyhjggjhgghhjhNWNWBWHWPPN222}
\begin{cases}
curl_{\vec x}\vec v= 0,\\
\frac{\partial\vec v}{\partial t}+d_\vec x\vec v\cdot\vec v=
-\nabla_{\vec x}\Phi,
\end{cases}
\end{equation}
Clearly the law
\er{MaxVacFull1ninshtrgravortghhghgjkgghklhjgkghghjjkjhjkkggjkhjkhjjhhfhjhklkhkhjjklzzzyyyhjggjhgghhjhNWNWBWHWPPN222}
is invariant under the change of inertial coordinate system given by
\er{noninchgravortbstrjgghguittu1}. Note also that, since in the
system $(*)$ we have $curl_{\vec x}\vec v=0$, we can write
\er{MaxVacFull1ninshtrgravortghhghgjkgghklhjgkghghjjkjhjkkggjkhjkhjjhhfhjhklkhkhjjklzzzyyyhjggjhgghhjhNWNWBWHWPPN222kgghjghjghj}
as the following Hamilton-Jacobi type equation:
\begin{equation}
\label{MaxVacFull1ninshtrgravortghhghgjkgghklhjgkghghjjkjhjkkggjkhjkhjjhhfhjhklkhkhjjklzzzyyyhjggjhgghhjhNWNWNWNWNWBWHWPPN222}
\begin{cases}
\vec v=\nabla_{\vec x}Z,\\
\frac{\partial Z}{\partial t}+\frac{1}{2}\left|\nabla_{\vec
x}Z\right|^2=-\Phi,
\end{cases}
\end{equation}
where $Z:=Z(\vec x,t)$ is some scalar field. We would like to derive
the law which is invariant in every non-inertial cartesian
coordinate system and is equivalent to
\er{MaxVacFull1ninshtrgravortghhghgjkgghklhjgkghghjjkjhjkkggjkhjkhjjhhfhjhklkhkhjjklzzzyyyhjggjhgghhjhNWNWBWHWPPN222}
in every inertial coordinate system. Note that
\er{MaxVacFull1ninshtrgravortghhghgjkgghklhjgkghghjjkjhjkkggjkhjkhjjhhfhjhklkhkhjjklzzzyyyhjggjhgghhjhNWNWBWHWPPN222}
and
\er{MaxVacFull1ninshtrgravortghhghgjkgghklhjgkghghjjkjhjkkggjkhjkhjjhhfhjhklkhkhjjklzzzyyyhjggjhgghhjhNWNWNWBWHWPPN222}
implies:
\begin{equation}
\label{MaxVacFull1ninshtrgravortghhghgjkgghklhjgkghghjjkjhjkkggjkhjkhjjhhfhjhklkhkhjjklzzzyyyhjfgfkjgjNWBWHWPPN222}
\begin{cases}
curl_{\vec x}\left(curl_{\vec x}\vec v\right)= 0,\\
div_{\vec x}\left\{\frac{\partial\vec v}{\partial t}+d_\vec x\vec
v\cdot\vec v+\frac{1}{2}\vec v\times curl_{\vec x}\vec v\right\}
= -4\pi GM,
\end{cases}
\end{equation}
that we rewrite using \er{apfrm3} as:
\begin{equation}
\label{MaxVacFull1ninshtrgravortghhghgjkgghklhjgkghghjjkjhjkkggjkhjkhjjhhfhjhklkhkhjjklzzzyyyNWBWHWPPN222}
\begin{cases}
curl_{\vec x}\left(curl_{\vec x}\vec v\right)= 0,\\
\frac{\partial}{\partial t}\left(div_{\vec x}\vec v\right)+\vec
v\cdot\nabla_{\vec x}\left(div_{\vec x}\vec
v\right)+\frac{1}{4}\left|d_{\vec x}\vec v+\{d_{\vec x}\vec
v\}^T\right|^2= -4\pi GM,
\end{cases}
\end{equation}
or, equivalently, as:
\begin{equation}
\label{MaxVacFull1ninshtrgravortghhghgjkgghklhjgkghghjjkjhjkkggjkhjkhjjhhfhjhklkhkhjjklzzzyyyNWBWHWPPN222ll}
\begin{cases}
curl_{\vec x}\left(curl_{\vec x}\vec v\right)= 0,\\
\frac{\partial}{\partial t}\left(div_{\vec x}\vec v\right)+div_{\vec
x}\left\{\left(div_{\vec x}\vec v\right)\vec
v\right\}+\frac{1}{4}\left|d_{\vec x}\vec v+\{d_{\vec x}\vec
v\}^T\right|^2-\left(div_{\vec x}\vec v\right)^2= -4\pi GM.
\end{cases}
\end{equation}
Next observe that using Proposition \ref{yghgjtgyrtrt}
we deduce that the laws in
\er{MaxVacFull1ninshtrgravortghhghgjkgghklhjgkghghjjkjhjkkggjkhjkhjjhhfhjhklkhkhjjklzzzyyyNWBWHWPPN222}
and
\er{MaxVacFull1ninshtrgravortghhghgjkgghklhjgkghghjjkjhjkkggjkhjkhjjhhfhjhklkhkhjjklzzzyyyNWBWHWPPN222ll}
are invariant under the change of non-inertial cartesian coordinate
system, given by \er{noninchgravortbstrjgghguittu2}. So, instead of
condition
\er{MaxVacFull1ninshtrgravortghhghgjkgghklhjgkghghjjkjhjkkggjkhjkhjjhhfhjhklkhkhjjklzzzyyyhjggjhgghhjhNWNWBWHWPPN222},
we can consider the frame independent condition
\er{MaxVacFull1ninshtrgravortghhghgjkgghklhjgkghghjjkjhjkkggjkhjkhjjhhfhjhklkhkhjjklzzzyyyNWBWHWPPN222ll}
together with the frame independent requirement, that the vectorial
gravitational potential $\vec v$ is an \underline{asymptotically
acceptable} speed-like vector field (see Definition
\ref{jhjkhhjhjhgjg;kkljkl} and Proposition
\ref{jhjkhhjhjhgjgjjkjkjklljj}). Then we can prove the following:
\begin{proposition}\label{jghjgyfyffffhfgljjkjl}
Let $\vec v:=\vec v(\vec x,t)$ be the vectorial gravitational
potential which is assumed to be an asymptotically acceptable
speed-like vector field (see Definition \ref{jhjkhhjhjhgjg;kkljkl}
and Proposition \ref{jhjkhhjhjhgjgjjkjkjklljj}) and let
$\Phi:=\Phi(\vec x,t)$ be a scalar Newtonian gravitational potential
which is assumed to be a proper scalar, which satisfies
\begin{equation}
\label{MaxVacFull1ninshtrgravortghhghgjkgghklhjgkghghjjkjhjkkggjkhjkhjjhhfhjhklkhkhjjklzzzyyyhjggjhgghhjhNWNWNWBWHWPPN222jkkhh}
\lim\limits_{|\vec x|\to+\infty}\nabla_{\vec x}\Phi(\vec
x,t)=0\quad\quad\quad\forall t.
\end{equation}
Furthermore, let $(*)$ be some cartesian coordinate system. Then:
\begin{itemize}
\item
If the coordinate system $(*)$ is \underline{non-rotating} with
respect to the speed-like vector field $\vec v$ (see Definition
\ref{jhjkhhjhjhgjg;kkljkljkkjlhlhhlh}) then in the system $(*)$ we
have
\begin{equation}\label{hjkggjjkhjh}
curl_{\vec x}\vec v= 0\quad\quad\forall (\vec x,t),
\end{equation}
if and only if
\begin{equation}\label{hjkggjkjtujhkdrysdutg8yfgjhi}
curl_{\vec x}\left(curl_{\vec x}\vec v\right)= 0\quad\quad\forall
(\vec x,t).
\end{equation}
\item
If, in addition, the coordinate system $(*)$ is \underline{inertial}
with respect to the speed-like vector field $\vec v$ (see Definition
\ref{jhjkhhjhjhgjg;kkljkljkkjlhlhhlh}) then in the system $(*)$ we
have
\begin{equation}\label{hjkggjjkhjh22}
\begin{cases}
curl_{\vec x}\vec v= 0\quad\quad\forall (\vec x,t),\\
\frac{\partial\vec v}{\partial t}+d_\vec x\vec v\cdot\vec v=
-\nabla_{\vec x}\Phi\quad\quad\forall (\vec x,t),
\end{cases}
\end{equation}
if and only if
\begin{equation}\label{hjkggjkjtujhkdrysdutg8yfgjhi22}
\begin{cases}
curl_{\vec x}\left(curl_{\vec x}\vec v\right)= 0\quad\quad\forall (\vec x,t),\\
\frac{\partial}{\partial t}\left(div_{\vec x}\vec v\right)+div_{\vec
x}\left\{\left(div_{\vec x}\vec v\right)\vec
v\right\}+\frac{1}{4}\left|d_{\vec x}\vec v+\{d_{\vec x}\vec
v\}^T\right|^2-\left(div_{\vec x}\vec v\right)^2= -\Delta_{\vec
x}\Phi\quad\quad\forall (\vec x,t).
\end{cases}
\end{equation}
\end{itemize}
\end{proposition}
\begin{proof}
Obviously, \er{hjkggjjkhjh} implies
\er{hjkggjkjtujhkdrysdutg8yfgjhi}. On the other hand, in the case
that the coordinate system $(*)$ is non-rotating with respect to the
speed-like vector field $\vec v$, by Definition
\ref{jhjkhhjhjhgjg;kkljkljkkjlhlhhlh} in the system $(*)$ we have
\begin{equation} \label{jljjjjkhhhhkhhjhjkjkjkljojkjkjkhj1jjkjkhkgh11vv}
\lim\limits_{|\vec x|\to+\infty}\left(d_{\vec x}\vec v(\vec
x,t)\right)
=0\quad\quad\quad\forall t.
\end{equation}
Therefore, if \er{hjkggjkjtujhkdrysdutg8yfgjhi} holds, then by
\er{hjkggjkjtujhkdrysdutg8yfgjhi} and
\er{jljjjjkhhhhkhhjhjkjkjkljojkjkjkhj1jjkjkhkgh11vv} in the
coordinate system $(*)$ we deduce
\begin{equation}
\label{MaxVacFull1ninshtrgravortghhghgjkgghklhjgkghghjjkjhjkkggjkhjkhjjhhfhjhklkhkhjjklzzzyyyNWBWHWPPN222lljkhjhjhhjuighg}
\begin{cases}
\Delta_{\vec x}\left(curl_{\vec x}\vec v(\vec
x,t)\right)= 0\quad\quad\quad\forall\, (\vec x,t),\\
\lim\limits_{|\vec x|\to+\infty}\left(curl_{\vec x}\vec v(\vec
x,t)\right)=0\quad\quad\quad\forall t.
\end{cases}
\end{equation}
However, by Liouville's theorem for the Laplace equation,
\er{MaxVacFull1ninshtrgravortghhghgjkgghklhjgkghghjjkjhjkkggjkhjkhjjhhfhjhklkhkhjjklzzzyyyNWBWHWPPN222lljkhjhjhhjuighg}
implies that in the coordinate system $(*)$ we also have
\begin{equation}
\label{MaxVacFull1ninshtrgravortghhghgjkgghklhjgkghghjjkjhjkkggjkhjkhjjhhfhjhklkhkhjjklzzzyyyNWBWHWPPN222lljkhjhjhhjuighgloojjkjl}
curl_{\vec x}\vec v(\vec x,t)=0\quad\quad\quad\forall\, (\vec x,t).
\end{equation}
So for that case \er{hjkggjjkhjh} and
\er{hjkggjkjtujhkdrysdutg8yfgjhi} are indeed equivalent.

Next, if, in addition, the coordinate system $(*)$ is inertial with
respect to the speed-like vector field $\vec v$, then by Definition
\ref{jhjkhhjhjhgjg;kkljkljkkjlhlhhlh} in the system $(*)$ we have
\begin{align} \label{jljjjjkhhhhkhhjhjkjkjkljojkjkjkhj1jjkjkhkgh11}
\lim\limits_{|\vec x|\to+\infty}\left(d_{\vec x}\vec v(\vec
x,t)\right)
=0\quad\quad\quad\forall t\,,\\
\label{jljjjjkhhhhkhhjhjkjkjkljojkjkjkhj211} \lim\limits_{|\vec
x|\to+\infty}\vec v(\vec x,t)=\vec d\quad\quad\quad\forall t\,,
\end{align}
where $\vec d\in \mathbb{R}^3$ is a constant (independent of $t$)
vector. In particular, by
\er{MaxVacFull1ninshtrgravortghhghgjkgghklhjgkghghjjkjhjkkggjkhjkhjjhhfhjhklkhkhjjklzzzyyyhjggjhgghhjhNWNWNWBWHWPPN222jkkhh},
\er{jljjjjkhhhhkhhjhjkjkjkljojkjkjkhj1jjkjkhkgh11} and
\er{jljjjjkhhhhkhhjhjkjkjkljojkjkjkhj211} in system $(*)$ we also
have
\begin{equation}
\label{MaxVacFull1ninshtrgravortghhghgjkgghklhjgkghghjjkjhjkkggjkhjkhjjhhfhjhklkhkhjjklzzzyyyhjfgfkjgjNWBWHWPPN222jjhhkljjkkhgjjhfhkjhjgjgj}
\lim\limits_{|\vec x|\to+\infty}\left(\frac{\partial\vec v}{\partial
t}+d_{\vec x}\vec v\cdot\vec v+\nabla_{\vec
x}\Phi\right)=0\quad\quad\quad\forall t.
\end{equation}
On the other hand, using \er{apfrm3} and \er{apfrm9} we deduce
\begin{multline}
\label{MaxVacFull1ninshtrgravortghhghgjkgghklhjgkghghjjkjhjkkggjkhjkhjjhhfhjhklkhkhjjklzzzyyyhjfgfkjgjNWBWHWPPN222jkjkkhhk}
\left(\frac{\partial}{\partial t}\left(div_{\vec x}\vec
v\right)+div_{\vec x}\left\{\left(div_{\vec x}\vec v\right)\vec
v\right\}+\frac{1}{4}\left|d_{\vec x}\vec v+\{d_{\vec x}\vec
v\}^T\right|^2-\left(div_{\vec x}\vec v\right)^2\right)-\frac{1}{2}\vec v\cdot curl_{\vec x}\left(curl_{\vec x}\vec v\right)=\\
\frac{\partial}{\partial t}\left(div_{\vec x}\vec v\right)+\vec
v\cdot\nabla_{\vec x}\left(div_{\vec x}\vec
v\right)+\frac{1}{4}\left|d_{\vec x}\vec v+\{d_{\vec x}\vec
v\}^T\right|^2-\frac{1}{2}\vec v\cdot curl_{\vec x}\left(curl_{\vec
x}\vec
v\right)=\\
div_{\vec x}\left\{\frac{\partial\vec v}{\partial t}+d_\vec x\vec
v\cdot\vec v\right\}+\frac{1}{4}\left|d_{\vec x}\vec v-\{d_{\vec
x}\vec v\}^T\right|^2-\frac{1}{2}\vec v\cdot curl_{\vec
x}\left(curl_{\vec x}\vec v\right)
=\\
div_{\vec x}\left\{\frac{\partial\vec v}{\partial t}+d_\vec x\vec
v\cdot\vec v\right\}+\frac{1}{2}\left|curl_{\vec x}\vec
v\right|^2-\frac{1}{2}\vec v\cdot curl_{\vec x}\left(curl_{\vec
x}\vec v\right)=
\\ div_{\vec x}\left\{\frac{\partial\vec v}{\partial t}+d_\vec x\vec
v\cdot\vec v+\frac{1}{2}\vec v\times curl_{\vec x}\vec v\right\} =
div_{\vec x}\left\{\frac{\partial\vec v}{\partial t}+\nabla_\vec
x\left(\frac{1}{2}\left|\vec v\right|^2\right)-\frac{1}{2}\vec
v\times curl_{\vec x}\vec v\right\}\,.
\end{multline}
In particular, by
\er{MaxVacFull1ninshtrgravortghhghgjkgghklhjgkghghjjkjhjkkggjkhjkhjjhhfhjhklkhkhjjklzzzyyyhjfgfkjgjNWBWHWPPN222jkjkkhhk}
we deduce that \er{hjkggjjkhjh22} implies
\er{hjkggjkjtujhkdrysdutg8yfgjhi22}. On the other hand if
\er{hjkggjkjtujhkdrysdutg8yfgjhi22} satisfied, then by
\er{MaxVacFull1ninshtrgravortghhghgjkgghklhjgkghghjjkjhjkkggjkhjkhjjhhfhjhklkhkhjjklzzzyyyhjfgfkjgjNWBWHWPPN222jkjkkhhk}
and
\er{MaxVacFull1ninshtrgravortghhghgjkgghklhjgkghghjjkjhjkkggjkhjkhjjhhfhjhklkhkhjjklzzzyyyNWBWHWPPN222lljkhjhjhhjuighgloojjkjl}
we deduce
\begin{equation}
\label{MaxVacFull1ninshtrgravortghhghgjkgghklhjgkghghjjkjhjkkggjkhjkhjjhhfhjhklkhkhjjklzzzyyyhjfgfkjgjNWBWHWPPN222jjhhklff}
\begin{cases}
curl_{\vec x}\vec v= 0,\\
div_{\vec x}\left\{\frac{\partial\vec v}{\partial t}+\nabla_\vec
x\left(\frac{1}{2}\left|\vec v\right|^2\right)\right\}
=-\Delta_{\vec x}\Phi.
\end{cases}
\end{equation}
In particular, by
\er{MaxVacFull1ninshtrgravortghhghgjkgghklhjgkghghjjkjhjkkggjkhjkhjjhhfhjhklkhkhjjklzzzyyyhjfgfkjgjNWBWHWPPN222jjhhklff}
together with \er{apfrm12} we deduce
\begin{multline}
\label{MaxVacFull1ninshtrgravortghhghgjkgghklhjgkghghjjkjhjkkggjkhjkhjjhhfhjhklkhkhjjklzzzyyyhjfgfkjgjNWBWHWPPN222jjhhkl}
-\Delta_{\vec x}\left(\nabla_{\vec x}\Phi\right)=\nabla_{\vec
x}\left(div_{\vec x}\left\{\frac{\partial\vec v}{\partial
t}+\nabla_\vec x\left(\frac{1}{2}\left|\vec
v\right|^2\right)\right\}\right)=\nabla_{\vec x}\left(div_{\vec
x}\left\{\frac{\partial\vec v}{\partial
t}\right\}\right)+\Delta_\vec x\left(\nabla_{\vec
x}\left\{\frac{1}{2}\left|\vec v\right|^2\right\}\right)\\=
\Delta_{\vec x}\left(\frac{\partial\vec v}{\partial t}+\nabla_{\vec
x}\left(\frac{1}{2}\left|\vec v\right|^2\right)\right)=\Delta_{\vec
x}\left(\frac{\partial\vec v}{\partial t}+d_{\vec x}\vec v\cdot\vec
v\right).
\end{multline}
Therefore, by
\er{MaxVacFull1ninshtrgravortghhghgjkgghklhjgkghghjjkjhjkkggjkhjkhjjhhfhjhklkhkhjjklzzzyyyhjfgfkjgjNWBWHWPPN222jjhhkl}
together with asymptotic condition
\er{MaxVacFull1ninshtrgravortghhghgjkgghklhjgkghghjjkjhjkkggjkhjkhjjhhfhjhklkhkhjjklzzzyyyhjfgfkjgjNWBWHWPPN222jjhhkljjkkhgjjhfhkjhjgjgj}
we deduce
\begin{equation}
\label{MaxVacFull1ninshtrgravortghhghgjkgghklhjgkghghjjkjhjkkggjkhjkhjjhhfhjhklkhkhjjklzzzyyyhjfgfkjgjNWBWHWPPN222jjhhkljjkkhgjjhf}
\begin{cases}
\Delta_{\vec x}\left(\frac{\partial\vec v}{\partial t}+d_{\vec
x}\vec v\cdot\vec v+\nabla_{\vec x}\Phi\right)=0.
\quad\quad\quad\forall\, (\vec x,t),\\
\lim\limits_{|\vec x|\to+\infty}\left(\frac{\partial\vec v}{\partial
t}+d_{\vec x}\vec v\cdot\vec v+\nabla_{\vec
x}\Phi\right)=0\quad\quad\quad\forall t.
\end{cases}
\end{equation}
Thus, again by Liouville's theorem for the Laplace equation,
\er{MaxVacFull1ninshtrgravortghhghgjkgghklhjgkghghjjkjhjkkggjkhjkhjjhhfhjhklkhkhjjklzzzyyyhjfgfkjgjNWBWHWPPN222jjhhkljjkkhgjjhf}
implies that in the coordinate system $(*)$ we have
\begin{equation}
\label{MaxVacFull1ninshtrgravortghhghgjkgghklhjgkghghjjkjhjkkggjkhjkhjjhhfhjhklkhkhjjklzzzyyyhjggjhgghhjhNWBWHWPPN222}
\frac{\partial\vec v}{\partial t}+d_{\vec x}\vec v\cdot\vec
v+\nabla_{\vec x}\Phi=0\quad\quad\quad\forall\, (\vec x,t).
\end{equation}
Then in the system $(*)$
\er{MaxVacFull1ninshtrgravortghhghgjkgghklhjgkghghjjkjhjkkggjkhjkhjjhhfhjhklkhkhjjklzzzyyyhjggjhgghhjhNWBWHWPPN222}
and
\er{MaxVacFull1ninshtrgravortghhghgjkgghklhjgkghghjjkjhjkkggjkhjkhjjhhfhjhklkhkhjjklzzzyyyNWBWHWPPN222lljkhjhjhhjuighgloojjkjl}
indeed implies \er{hjkggjjkhjh22} and the result follows.
\end{proof}
\begin{definition}\label{jggggggggggggggggggggjklj}
We can call a cartesian coordinate system, which is inertial with
respect to the vectorial gravitational potential $\vec v$, by the
usual name an \underline{inertial} cartesian coordinate system.
Then, by Proposition \ref{jghfgjdhdhg} it is clear, that every
cartesian coordinate system $(**)$, that we can get from such an
inertial coordinate system $(*)$ by the usual Galilean
Transformations, and all equivalent coordinate systems also will be
inertial. Note also that by Proposition \ref{jghjgyfyffffhfgljjkjl},
in the case of the simplest Newtonian-like model of the gravity,
given by
\er{MaxVacFull1ninshtrgravortghhghgjkgghklhjgkghghjjkjhjkkggjkhjkhjjhhfhjhklkhkhjjklzzzyyyNWBWHWPPN222ll},
equalities in
\er{MaxVacFull1ninshtrgravortghhghgjkgghklhjgkghghjjkjhjkkggjkhjkhjjhhfhjhklkhkhjjklzzzyyyhjggjhgghhjhNWNWBWHWPPN222}
are valid in every inertial cartesian system. More generally, we
call every cartesian coordinate system, which is non-rotating with
respect to the vectorial gravitational potential $\vec v$, by the
usual name a \underline{non-rotating} cartesian coordinate system.
Note again that by Proposition \ref{jghjgyfyffffhfgljjkjl}, in the
case of the gravity given by
\er{MaxVacFull1ninshtrgravortghhghgjkgghklhjgkghghjjkjhjkkggjkhjkhjjhhfhjhklkhkhjjklzzzyyyNWBWHWPPN222ll},
equality \er{hjkggjjkhjh} holds in every non-rotating cartesian
system.
\end{definition}

As a consequence of all mentioned above, the second law of Newton,
invariant under the change of non-inertial cartesian coordinate
system, has the following form:
\begin{equation}\label{noninchgravortbstrjgghguittu2gjgghhjhghjhjgghgghghghtytythvfghfgghjggghjgjh}
m\frac{d^2\vec x}{dt^2}=m\frac{d\vec
u}{dt}=m\left(\frac{\partial\vec v}{\partial t}(\vec
x,t)+\frac{1}{2}\nabla_{\vec x}\left(\left|\vec
v\right|^2\right)(\vec x,t)\right)-m\vec u\times curl_{\vec x}\vec
v(\vec x,t)+\vec F,
\end{equation}
and the first approximation of the law of gravity, which is
equivalent to the Newtonian gravity and which is invariant under the
change of non-inertial cartesian coordinate system, has the
following form:
\begin{equation}
\label{MaxVacFull1ninshtrgravortghhghgjkgghklhjgkghghjjkjhjkkggjkhjkhjjhhfhjhklkhkhjjklzzzyyyNWBWHWPPN222jkgghgg}
\begin{cases}
curl_{\vec x}\left(curl_{\vec x}\vec v\right)= 0,\\
\frac{\partial}{\partial t}\left(div_{\vec x}\vec v\right)+div_{\vec
x}\left\{\left(div_{\vec x}\vec v\right)\vec
v\right\}+\frac{1}{4}\left|d_{\vec x}\vec v+\{d_{\vec x}\vec
v\}^T\right|^2-\left(div_{\vec x}\vec v\right)^2= -4\pi GM.
\end{cases}
\end{equation}
%
%
%
%
%
%
Here $\vec x:=\vec x(t)$, $\vec u:=\vec u(t)=\frac{d\vec x}{dt}(t)$
and $m$ are the place, the velocity and the inertial mass of some
given particle at the moment of time $t$, $\vec v:=\vec v(\vec x,t)$
is the vectorial gravitational potential, $M$ is the volume density
of gravitational masses and $\vec F$ is the total non-gravitational
force, acting on the given particle. Moreover, the vectorial
gravitational potential $\vec v$ is a speed-like vector field, i.e.
under the changes of inertial or non-inertial cartesian coordinate
system it behaves like a field of velocities of some continuum. Thus
we could introduce the fictitious continuum medium covering all the
space, that we can call Aether, such that $\vec v(\vec x,t)$ is a
fictitious velocity of this medium in the point $\vec x$ at the time
$t$.

\begin{remark}\label{ugyugg} In the case of Newtonian-like gravity, given by
\er{MaxVacFull1ninshtrgravortghhghgjkgghklhjgkghghjjkjhjkkggjkhjkhjjhhfhjhklkhkhjjklzzzyyyNWBWHWPPN222ll},
the quantity $Z(\vec x,t)$ in
\er{MaxVacFull1ninshtrgravortghhghgjkgghklhjgkghghjjkjhjkkggjkhjkhjjhhfhjhklkhkhjjklzzzyyyhjggjhgghhjhNWNWNWNWNWBWHWPPN222}
is well defined (modulo additive constant) in every non-rotating
cartesian coordinate system and in particular it is well defined
(modulo additive constant) in inertial coordinate systems. It can be
easily checked by straightforward calculations that, if under the
change of coordinate system $\bf{(*)}$ to $\bf{(**)}$ given by the
Galilean Transformation
\begin{equation}\label{noninchgravortbstrjgghguittu1intmm}
\begin{cases}
\vec x'=\vec x+\vec wt,\\
t'=t,
\end{cases}
\end{equation}
the quantity $Z$ transforms as:
\begin{equation}\label{noninchgravortbstrjgghguittu1intmmjhhj}
Z'(\vec x',t')=Z(\vec x,t)+\vec w\cdot\vec x+\frac{1}{2}|\vec
w|^2t\quad\text{(modulo unspecified additive constant, independent
on $(\vec x,t)$)},
\end{equation}
then equalities
\begin{equation}
\label{MaxVacFull1ninshtrgravortghhghgjkgghklhjgkghghjjkjhjkkggjkhjkhjjhhfhjhklkhkhjjklzzzyyyhjggjhgghhjhNWNWNWNWNWBWHWPPN222kkk}
\begin{cases}
\vec v+\vec w=\vec v'=\nabla_{\vec x'}Z',\\
\frac{\partial Z'}{\partial t'}+\frac{1}{2}\left|\nabla_{\vec
x'}Z'\right|^2=-\Phi',
\end{cases}
\end{equation}
in coordinate system $\bf{(**)}$ imply the similar equalities
\begin{equation}
\label{MaxVacFull1ninshtrgravortghhghgjkgghklhjgkghghjjkjhjkkggjkhjkhjjhhfhjhklkhkhjjklzzzyyyhjggjhgghhjhNWNWNWNWNWBWHWPPN222kkkkjj}
\begin{cases}
\vec v=\nabla_{\vec x}Z,\\
\frac{\partial Z}{\partial t}+\frac{1}{2}\left|\nabla_{\vec
x}Z\right|^2=-\Phi,
\end{cases}
\end{equation}
in coordinate system $\bf{(*)}$, provided that $\Phi'=\Phi$ (see
Lemma \ref{ygjhgjgjgjhgjtuiuijk} for the details). More generally,
by Lemma \ref{ygjhgjgjgjhgjtu} we deduce that we can choose the
quantity $Z$ to be a speed-like field generator (see Definition
\ref{gjgjffjjjjjjjjjjjjjjjjjjjjjjhh2}) on the extended Galilean
group of all inertial cartesian coordinate systems. Moreover, again
by Lemma \ref{ygjhgjgjgjhgjtuiuijk}, in that case the equations in
\er{MaxVacFull1ninshtrgravortghhghgjkgghklhjgkghghjjkjhjkkggjkhjkhjjhhfhjhklkhkhjjklzzzyyyhjggjhgghhjhNWNWNWNWNWBWHWPPN222}
preserve their form in every inertial cartesian coordinate system.
\end{remark}
\begin{remark}\label{gygygygyggjh}
One can wonder: what should be possible values of the vectorial
gravitational potential $\vec v$ in the proximity of the Earth or
another massive body? In order to try to answer this question in the
case of the Newtonian-type gravity , given by
\er{MaxVacFull1ninshtrgravortghhghgjkgghklhjgkghghjjkjhjkkggjkhjkhjjhhfhjhklkhkhjjklzzzyyyNWBWHWPPN222ll},
consider two cartesian coordinate systems: \underline{non-rotating}
system $\bf(*)$ with the center that coincides with the center of
masses of the Earth and \underline{inertial} system $\bf(**)$
related to some external cosmic bodies. Assume that the center of
masses of the Earth has place $\vec R(t')$ and velocity $\vec
W(t'):=\frac{d\vec R}{dt'}(t')$ in the coordinate system $\bf(**)$.
Thus the change of coordinate system $\bf(*)$ to coordinate system
$\bf(**)$ is given by
\begin{equation}\label{hjjgghghghyu}
\begin{cases}
\vec x'=\vec x+\vec R(t),\\
t'=t,
\end{cases}
\end{equation}
and the vectorial gravitational potential $\vec v$, being a speed
like vector field, transforms as
\begin{equation}\label{hjjgghghghyulll}
\vec v'=\vec v+\vec W(t).
\end{equation}
Next, since the system $\bf(**)$ is inertial, consistently with
\er{MaxVacFull1ninshtrgravortghhghgjkgghklhjgkghghjjkjhjkkggjkhjkhjjhhfhjhklkhkhjjklzzzyyyhjggjhgghhjhNWNWBWHWPPN222}
and
\er{MaxVacFull1ninshtrgravortghhghgjkgghklhjgkghghjjkjhjkkggjkhjkhjjhhfhjhklkhkhjjklzzzyyyhjggjhgghhjhNWNWNWBWHWPPN222}
we have
\begin{equation}
\label{MaxVacFull1ninshtrgravortghhghgjkgghklhjgkghghjjkjhjkkggjkhjkhjjhhfhjhklkhkhjjklzzzyyyhjggjhgghhjhNWNWBWHWPPN222kkk}
\begin{cases}
curl_{\vec x'}\vec v'= 0,\\
\frac{\partial\vec v'}{\partial t'}+d_{\vec x'}\vec v'\cdot\vec v'=
-\nabla_{\vec x'}\Phi'_1-\nabla_{\vec x'}\Phi'_2,
\end{cases}
\end{equation}
with
\begin{equation}
\label{MaxVacFull1ninshtrgravortghhghgjkgghklhjgkghghjjkjhjkkggjkhjkhjjhhfhjhklkhkhjjklzzzyyyhjggjhgghhjhNWNWNWBWHWPPN222kkk}
\Delta_{\vec x'}\Phi'_1=4\pi GM'_1\quad\text{and}\quad\Delta_{\vec
x'}\Phi'_2=4\pi GM'_2,
\end{equation}
where $M_1$ is the gravitational mass density of the Earth and $M_2$
is the gravitational mass density of all other external cosmic
bodies like sun et.al. Moreover, again since the system $\bf(**)$ is
inertial, we clearly have:
\begin{equation}
\label{MaxVacFull1ninshtrgravortghhghgjkgghklhjgkghghjjkjhjkkggjkhjkhjjhhfhjhklkhkhjjklzzzyyyhjggjhgghhjhNWNWBWHWPPN222kkkjhjhjjhjhjhjhhjhj}
\frac{d\vec W}{\partial t}(t)=\frac{d\vec W}{\partial
t'}(t')=-\nabla_{\vec x'}\Phi'_2\left(\vec R(t),t\right).
\end{equation}
On the other hand inserting \er{hjjgghghghyu} and
\er{hjjgghghghyulll} into
\er{MaxVacFull1ninshtrgravortghhghgjkgghklhjgkghghjjkjhjkkggjkhjkhjjhhfhjhklkhkhjjklzzzyyyhjggjhgghhjhNWNWBWHWPPN222kkk}
and using Proposition \ref{yghgjtgyrtrt}
we deduce
\begin{equation}
\label{MaxVacFull1ninshtrgravortghhghgjkgghklhjgkghghjjkjhjkkggjkhjkhjjhhfhjhklkhkhjjklzzzyyyhjggjhgghhjhNWNWBWHWPPN222kkkgtytghjjh}
curl_{\vec x}\vec v=0,
\end{equation}
and
\begin{equation}
\label{MaxVacFull1ninshtrgravortghhghgjkgghklhjgkghghjjkjhjkkggjkhjkhjjhhfhjhklkhkhjjklzzzyyyhjggjhgghhjhNWNWBWHWPPN222kkkgtytghjjhghhg}
\frac{d\vec W}{\partial t}(t)+\frac{\partial\vec v}{\partial
t}+d_{\vec x}\vec v\cdot\vec v=\frac{d\vec W}{\partial
t'}(t')+\frac{\partial\vec v}{\partial t'}+d_{\vec x'}\vec
v\cdot\vec W(t')+d_{\vec x'}\vec v\cdot\vec v= -\nabla_{\vec
x'}\Phi'_1-\nabla_{\vec x'}\Phi'_2= -\nabla_{\vec
x}\Phi_1-\nabla_{\vec x'}\Phi'_2,
\end{equation}
On the other hand, the quantities $\nabla_{\vec
x'}\Phi'_2=\nabla_{\vec x}\Phi_2$ and $\frac{d\vec W}{\partial
t}(t)$, being generated by the gravitational field from the far
bodies, are insignificant, at the first approximation, with respect
to the quantity $\nabla_{\vec x}\Phi_1$ in the scale compatible to
the Earth size. Moreover, even if we wish to consider a finer
approximation then we can observe that  the quantity $\nabla_{\vec
x'}\Phi'_2=\nabla_{\vec x}\Phi_2$ vary insignificantly in the space
variables in the scale compatible to the Earth size and thus, using
\er{MaxVacFull1ninshtrgravortghhghgjkgghklhjgkghghjjkjhjkkggjkhjkhjjhhfhjhklkhkhjjklzzzyyyhjggjhgghhjhNWNWBWHWPPN222kkkjhjhjjhjhjhjhhjhj}
we get
\begin{equation}
\label{hklgkjgjgjgfhfhffjjohiiy} \nabla_{\vec x}\Phi_2\left(\vec
x,t\right)=\nabla_{\vec x'}\Phi'_2\left(\vec x',t'\right)\approx
\nabla_{\vec x'}\Phi'_2\left(\vec R(t),t\right)=-\frac{d\vec
W}{\partial t}(t).
\end{equation}
Therefore, in both cases by
\er{MaxVacFull1ninshtrgravortghhghgjkgghklhjgkghghjjkjhjkkggjkhjkhjjhhfhjhklkhkhjjklzzzyyyhjggjhgghhjhNWNWBWHWPPN222kkkgtytghjjh},
\er{MaxVacFull1ninshtrgravortghhghgjkgghklhjgkghghjjkjhjkkggjkhjkhjjhhfhjhklkhkhjjklzzzyyyhjggjhgghhjhNWNWBWHWPPN222kkkgtytghjjhghhg}
and \er{hklgkjgjgjgfhfhffjjohiiy} we finally deduce
\begin{equation}
\label{MaxVacFull1ninshtrgravortghhghgjkgghklhjgkghghjjkjhjkkggjkhjkhjjhhfhjhklkhkhjjklzzzyyyhjggjhgghhjhNWNWBWHWPPN222kkkgtytghjjhpop}
\begin{cases}
curl_{\vec x}\vec v=0,\\
\frac{\partial\vec v}{\partial t}+\frac{1}{2}\nabla_{\vec
x}\left(|\vec v|^2\right)=\frac{\partial\vec v}{\partial t}+d_{\vec
x}\vec v\cdot\vec v\approx -\nabla_{\vec x}\Phi_1,
\end{cases}
\end{equation}
where
\begin{equation}
\label{MaxVacFull1ninshtrgravortghhghgjkgghklhjgkghghjjkjhjkkggjkhjkhjjhhfhjhklkhkhjjklzzzyyyhjggjhgghhjhNWNWNWBWHWPPN222kkkiklklk}
\Delta_{\vec x}\Phi_1=4\pi GM_1.
\end{equation}
Being in the system  $\bf(*)$ which is stationary with respect to
the center of the Earth we look for stationary (i.e. time
independent) solutions of
\er{MaxVacFull1ninshtrgravortghhghgjkgghklhjgkghghjjkjhjkkggjkhjkhjjhhfhjhklkhkhjjklzzzyyyhjggjhgghhjhNWNWBWHWPPN222kkkgtytghjjhpop}.
Thus
\er{MaxVacFull1ninshtrgravortghhghgjkgghklhjgkghghjjkjhjkkggjkhjkhjjhhfhjhklkhkhjjklzzzyyyhjggjhgghhjhNWNWBWHWPPN222kkkgtytghjjhpop}
implies:
\begin{equation}
\label{MaxVacFull1ninshtrgravortghhghgjkgghklhjgkghghjjkjhjkkggjkhjkhjjhhfhjhklkhkhjjklzzzyyyhjggjhgghhjhNWNWBWHWPPN222kkkgtytghjjhpoppkkk}
\begin{cases}
curl_{\vec x}\vec v(\vec x)=0,\\
\left|\vec v(\vec x)\right|^2= -2\Phi_1(\vec x).
\end{cases}
\end{equation}
On the other hand, the scalar field $\Phi_1$, being the Newtonian
potential of the Earth, is radial and outside of the Earth surface
it is known that $\Phi_1(\vec x)=-\frac{Gm_0}{|\vec x|}$, where
$m_0$ is the Earth mass. Thus, since there exists a scalar field
$Z_0(\vec x)$ such that $\vec v(\vec x)=\nabla_{\vec x}Z_0(\vec x)$
and since of symmetry considerations $Z_0(\vec x)=Z_0(|\vec x|)$
should be radial, by
\er{MaxVacFull1ninshtrgravortghhghgjkgghklhjgkghghjjkjhjkkggjkhjkhjjhhfhjhklkhkhjjklzzzyyyhjggjhgghhjhNWNWBWHWPPN222kkkgtytghjjhpoppkkk}
we obtain
\begin{equation}
\label{MaxVacFull1ninshtrgravortghhghgjkgghklhjgkghghjjkjhjkkggjkhjkhjjhhfhjhklkhkhjjklzzzyyyhjggjhgghhjhNWNWBWHWPPN222kkkgtytghjjhpoppkkkhhhk}
\left|\frac{dZ_0}{d(|\vec x|)}(|\vec x|)\right|= \sqrt{-2\Phi_1(\vec
x)},
\end{equation}
that implies either
\begin{equation}
\label{MaxVacFull1ninshtrgravortghhghgjkgghklhjgkghghjjkjhjkkggjkhjkhjjhhfhjhklkhkhjjklzzzyyyhjggjhgghhjhNWNWBWHWPPN222kkkgtytghjjhpoppkkkhhhkhjhj}
\vec v(\vec x)= \frac{\sqrt{-2\Phi_1(|\vec x|)}}{|\vec x|}\vec x,
\end{equation}
or
\begin{equation}
\label{MaxVacFull1ninshtrgravortghhghgjkgghklhjgkghghjjkjhjkkggjkhjkhjjhhfhjhklkhkhjjklzzzyyyhjggjhgghhjhNWNWBWHWPPN222kkkgtytghjjhpoppkkkhhhkhjhjiuu}
\vec v(\vec x)= -\frac{\sqrt{-2\Phi_1(|\vec x|)}}{|\vec x|}\vec x.
\end{equation}
In particular on the Earth surface we have:
\begin{equation}
\label{MaxVacFull1ninshtrgravortghhghgjkgghklhjgkghghjjkjhjkkggjkhjkhjjhhfhjhklkhkhjjklzzzyyyhjggjhgghhjhNWNWBWHWPPN222kkkgtytghjjhpoppkkkhhhkhjhjiuuokokok}
|\vec v|=\sqrt{\frac{2Gm_0}{r_0}},
\end{equation}
where $r_0$ is the Earth radius and $m_0$ is the Earth mass, i.e.
the absolute value of the vectorial gravitational potential on the
Earth surface approximately equals to the escape velocity and its
direction is normal to the Earth, either downward or upward.
\end{remark}

\subsection{Variational principle for the solution of the Cauchy-problem for Hamilton-Jacobi type equation
\er{MaxVacFull1ninshtrgravortghhghgjkgghklhjgkghghjjkjhjkkggjkhjkhjjhhfhjhklkhkhjjklzzzyyyhjggjhgghhjhNWNWNWNWNWBWHWPPN222}.}\label{GGGGGZZZZffffhhh}
Consider the Hamilton-Jacobi type equation
\er{MaxVacFull1ninshtrgravortghhghgjkgghklhjgkghghjjkjhjkkggjkhjkhjjhhfhjhklkhkhjjklzzzyyyhjggjhgghhjhNWNWNWNWNWBWHWPPN222}
for the vectorial gravitational potential in the non-rotating
coordinate system:
\begin{equation}
\label{MaxVacFull1ninshtrgravortghhghgjkgghklhjgkghghjjkjhjkkggjkhjkhjjhhfhjhklkhkhjjklzzzyyyhjggjhgghhjhNWNWNWNWNWBWHWPPN222RR}
\begin{cases}
\vec v=\nabla_{\vec x}Z,\\
\frac{\partial Z}{\partial t}+\frac{1}{2}\left|\nabla_{\vec
x}Z\right|^2=-\Phi,
\end{cases}
\end{equation}
where $Z:=Z(\vec x,t)$ is some scalar field and $\Phi=\Phi(\vec
x,t)$ is the scalar Newtonian gravitational potential which
satisfies
\er{MaxVacFull1ninshtrgravortghhghgjkgghklhjgkghghjjkjhjkkggjkhjkhjjhhfhjhklkhkhjjklzzzyyyhjggjhgghhjhNWNWNWBWHWPPN222}:
\begin{equation}
\label{MaxVacFull1ninshtrgravortghhghgjkgghklhjgkghghjjkjhjkkggjkhjkhjjhhfhjhklkhkhjjklzzzyyyhjggjhgghhjhNWNWNWBWHWPPN222RR}
\Delta_{\vec x}\Phi=4\pi GM,
\end{equation}
where $M$ is the gravitational mass density and $G$ is the
gravitational constant. Next consider the Cauchy problem for
\er{MaxVacFull1ninshtrgravortghhghgjkgghklhjgkghghjjkjhjkkggjkhjkhjjhhfhjhklkhkhjjklzzzyyyhjggjhgghhjhNWNWNWNWNWBWHWPPN222RR}:
\begin{equation}
\label{MaxVacFull1ninshtrgravortghhghgjkgghklhjgkghghjjkjhjkkggjkhjkhjjhhfhjhklkhkhjjklzzzyyyhjggjhgghhjhNWNWNWNWNWBWHWPPN222RRkmmjljlm}
\begin{cases}
\frac{\partial Z}{\partial t}(\vec
x,t)+\frac{1}{2}\left|\nabla_{\vec x}Z(\vec x,t)\right|^2=-\Phi(\vec
x,t),\\
Z(\vec x,0)=\varphi(\vec x),
\end{cases}
\end{equation}
where $\Phi(\vec x,t)$ and $\varphi(\vec x)$ are prescribed smooth
scalar functions.
Furthermore, for
every instant of time $t\geq 0$ consider the followed variational
functional defined on trajectories $\vec r(s):[0,t]\to\R^3$:
\begin{equation}\label{gjgjfhhhhhhhhhhhhhhhhhjjjjk}
J_{gr}\left(\vec r\right)=\varphi\left(\vec
r(0)\right)+\int_0^t\left(\frac{1}{2}\left|\frac{d\vec
r}{ds}(s)\right|^2-\Phi\left(\vec r(s),s\right)\right)ds.
\end{equation}
Then inserting the smooth solution $Z$ of
\er{MaxVacFull1ninshtrgravortghhghgjkgghklhjgkghghjjkjhjkkggjkhjkhjjhhfhjhklkhkhjjklzzzyyyhjggjhgghhjhNWNWNWNWNWBWHWPPN222RRkmmjljlm}
into \er{gjgjfhhhhhhhhhhhhhhhhhjjjjk} and using the Chain Rule we
deduce:
\begin{multline}\label{gjgjfhhhhhhhhhhhhhhhhhjjjjklmjkkhk}
J_{gr}\left(\vec r\right)=Z\left(\vec
r(0),0\right)+\int_0^t\left(\frac{1}{2}\left|\frac{d\vec
r}{ds}(s)\right|^2+\frac{\partial Z}{\partial s}\left(\vec
r(s),s\right)+\frac{1}{2}\left|\nabla_{\vec {\vec r}}Z\left(\vec
r(s),s\right)\right|^2\right)ds=\\ Z\left(\vec
r(0),0\right)+\int_0^t\left(\frac{1}{2}\left|\frac{d\vec
r}{ds}(s)\right|^2+\frac{d}{ds}\left(Z\left(\vec
r(s),s\right)\right)-\nabla_{\vec {\vec r}}Z\left(\vec
r(s),s\right)\cdot\frac{d\vec
r}{ds}(s)+\frac{1}{2}\left|\nabla_{\vec {\vec r}}Z\left(\vec
r(s),s\right)\right|^2\right)ds
\\ =Z\left(\vec
r(0),0\right)+\int_0^t\frac{d}{ds}\left(Z\left(\vec
r(s),s\right)\right)ds +\frac{1}{2}\int_0^t\left|\frac{d\vec
r}{ds}(s)-\nabla_{\vec {\vec r}}Z\left(\vec r(s),s\right)\right|^2ds
\\ =Z\left(\vec
r(0),0\right)+\left(Z\left(\vec r(t),t\right)-Z\left(\vec
r(0),0\right)\right) +\frac{1}{2}\int_0^t\left|\frac{d\vec
r}{ds}(s)-\nabla_{\vec {\vec r}}Z\left(\vec r(s),s\right)\right|^2ds
\\
=Z\left(\vec r(t),t\right)+\frac{1}{2}\int_0^t\left|\frac{d\vec
r}{ds}(s)-\nabla_{\vec {\vec r}}Z\left(\vec
r(s),s\right)\right|^2ds.
\end{multline}
Thus we rewrite the definition of the functional $J_{gr}$ in
\er{gjgjfhhhhhhhhhhhhhhhhhjjjjk} as:
\begin{multline}\label{gjgjfhhhhhhhhhhhhhhhhhjjjjklmjkkhkjjjjmjlj}
J_{gr}\left(\vec r\right)=Z\left(\vec
r(t),t\right)+\frac{1}{2}\int_0^t\left|\frac{d\vec
r}{ds}(s)-\nabla_{\vec {\vec r}}Z\left(\vec
r(s),s\right)\right|^2ds=Z\left(\vec
r(t),t\right)+\frac{1}{2}\int_0^t\left|\frac{d\vec r}{ds}(s)-\vec
v\left(\vec r(s),s\right)\right|^2ds,
\end{multline}
where consistently with
\er{MaxVacFull1ninshtrgravortghhghgjkgghklhjgkghghjjkjhjkkggjkhjkhjjhhfhjhklkhkhjjklzzzyyyhjggjhgghhjhNWNWNWNWNWBWHWPPN222RR}
we consider $\vec v(\vec x,t):=\nabla_{\vec x}Z(\vec x,t)$.
Therefore, if for every point $\vec x\in\R^3$ and every instant of
time $t\geq 0$ we consider the function $g:=g(\vec
x,t):\R^3:[0,+\infty)\to\R$ defined as a minimum of the following
variational problem:
\begin{multline}\label{gjgjfhhhhhhhhhhhhhhhhhjjjjkhhhjhjjhh}
g(\vec x,t):=\min{\left\{J_{gr}\left(\vec r\right)\;\mid\;\;\vec
r(s):[0,t]\to\R^3,\;r(t)=\vec x\right\}}=\\
\min{\left\{\varphi\left(\vec
r(0)\right)+\int_0^t\left(\frac{1}{2}\left|\frac{d\vec
r}{ds}(s)\right|^2-\Phi\left(\vec
r(s),s\right)\right)ds\;\mid\;\;\vec r(s):[0,t]\to\R^3,\;\vec
r(t)=\vec x\right\}},
\end{multline}
then by \er{gjgjfhhhhhhhhhhhhhhhhhjjjjklmjkkhkjjjjmjlj} we deduce
that
\begin{equation}\label{gjgjfhhhhhhhhhhhhhhhhhjjjjkhhhjhjjhhjkjhh}
g(\vec x,t)=Z(\vec x,t),
\end{equation}
and this minimum is achieved on the unique trajectory $\vec
r(s):[0,t]\to\R^3$ that satisfies the following initial value
problem for an ordinary differential equation:
\begin{equation}
\label{MaxVacFull1ninshtrgravortghhghgjkgghklhjgkghghjjkjhjkkggjkhjkhjjhhfhjhklkhkhjjklzzzyyyhjggjhgghhjhNWNWNWNWNWBWHWPPN222RRkkljljkhhjh}
\begin{cases}
\frac{d\vec r}{ds}(s)=\nabla_{\vec {\vec r}}Z\left(\vec
r(s),s\right)=\vec v\left(\vec r(s),s\right)\quad\quad\forall
s\in[0,t],\\
\vec r(t)=\vec x.
\end{cases}
\end{equation}
So by inserting \er{gjgjfhhhhhhhhhhhhhhhhhjjjjkhhhjhjjhhjkjhh} into
\er{gjgjfhhhhhhhhhhhhhhhhhjjjjkhhhjhjjhh} we obtain the
\underline{explicit} formula for the solution $Z$ of the Cauchy
problem
\er{MaxVacFull1ninshtrgravortghhghgjkgghklhjgkghghjjkjhjkkggjkhjkhjjhhfhjhklkhkhjjklzzzyyyhjggjhgghhjhNWNWNWNWNWBWHWPPN222RRkmmjljlm}:
\begin{multline}\label{gjgjfhhhhhhhhhhhhhhhhhjjjjkhhhjhjjhhggg}
Z(\vec x,t)=\min{\left\{\varphi\left(\vec
r(0)\right)+\int_0^t\left(\frac{1}{2}\left|\frac{d\vec
r}{ds}(s)\right|^2-\Phi\left(\vec
r(s),s\right)\right)ds\;\mid\;\;\vec r(s):[0,t]\to\R^3,\;\vec
r(t)=\vec x\right\}}.
\end{multline}
Furthermore, there exists a unique trajectory $\vec
r(s):[0,t]\to\R^3$ that minimizes the right-hand-side of
\er{gjgjfhhhhhhhhhhhhhhhhhjjjjkhhhjhjjhhggg} and moreover it
satisfies:
\begin{equation}
\label{MaxVacFull1ninshtrgravortghhghgjkgghklhjgkghghjjkjhjkkggjkhjkhjjhhfhjhklkhkhjjklzzzyyyhjggjhgghhjhNWNWNWNWNWBWHWPPN222RRkkljljkhhjhjhhhjh}
\begin{cases}
\frac{d\vec r}{ds}(s)=\nabla_{\vec {\vec r}}Z\left(\vec
r(s),s\right)=\vec v\left(\vec r(s),s\right)\quad\quad\forall
s\in[0,t],\\
\vec r(t)=\vec x.
\end{cases}
\end{equation}
In particular for this trajectory we have
\begin{equation}\label{MaxVacFull1ninshtrgravortghhghgjkgghklhjgkghghjjkjhjkkggjkhjkhjjhhfhjhklkhkhjjklzzzyyyhjggjhgghhjhNWNWNWNWNWBWHWPPN222RRkkljljkhhjhjhhhjhkjjg}
\vec v\left(\vec x,t\right)=\frac{d\vec r}{ds}(t).
\end{equation}
On the other hand the minimizer of the right-hand-side of
\er{gjgjfhhhhhhhhhhhhhhhhhjjjjkhhhjhjjhhggg} clearly satisfies the
following Euler-Lagrange equation:
\begin{equation}
\label{MaxVacFull1ninshtrgravortghhghgjkgghklhjgkghghjjkjhjkkggjkhjkhjjhhfhjhklkhkhjjklzzzyyyhjggjhgghhjhNWNWNWNWNWBWHWPPN222RRkkljljkhhjhjhhhjhjkh}
\begin{cases}
\frac{d^2\vec r}{ds^2}(s)=-\nabla_{\vec {\vec r}}\Phi\left(\vec
r(s),s\right)\quad\quad\forall
s\in[0,t],\\
\frac{d\vec r}{ds}(0)=\nabla_{\vec r}\varphi\left(\vec r(0)\right),
\\
\vec r(t)=\vec x.
\end{cases}
\end{equation}
Equation
\er{MaxVacFull1ninshtrgravortghhghgjkgghklhjgkghghjjkjhjkkggjkhjkhjjhhfhjhklkhkhjjklzzzyyyhjggjhgghhjhNWNWNWNWNWBWHWPPN222RRkkljljkhhjhjhhhjhjkh}
sometimes has an advantage with respect to
\er{MaxVacFull1ninshtrgravortghhghgjkgghklhjgkghghjjkjhjkkggjkhjkhjjhhfhjhklkhkhjjklzzzyyyhjggjhgghhjhNWNWNWNWNWBWHWPPN222RRkkljljkhhjhjhhhjh}
since the a priori unknown function $Z$ dose not appear in
\er{MaxVacFull1ninshtrgravortghhghgjkgghklhjgkghghjjkjhjkkggjkhjkhjjhhfhjhklkhkhjjklzzzyyyhjggjhgghhjhNWNWNWNWNWBWHWPPN222RRkkljljkhhjhjhhhjhjkh}.

So, in order to find $\vec v(\vec x,t)$ we need first to solve
\er{MaxVacFull1ninshtrgravortghhghgjkgghklhjgkghghjjkjhjkkggjkhjkhjjhhfhjhklkhkhjjklzzzyyyhjggjhgghhjhNWNWNWNWNWBWHWPPN222RRkkljljkhhjhjhhhjhjkh}
and then use
\er{MaxVacFull1ninshtrgravortghhghgjkgghklhjgkghghjjkjhjkkggjkhjkhjjhhfhjhklkhkhjjklzzzyyyhjggjhgghhjhNWNWNWNWNWBWHWPPN222RRkkljljkhhjhjhhhjhkjjg}
with the solution of
\er{MaxVacFull1ninshtrgravortghhghgjkgghklhjgkghghjjkjhjkkggjkhjkhjjhhfhjhklkhkhjjklzzzyyyhjggjhgghhjhNWNWNWNWNWBWHWPPN222RRkkljljkhhjhjhhhjhjkh}
that we just found.

\subsection{Genuine gravity and inertia}\label{jlhkgkhfgjdfhd}
\begin{definition}\label{jjhgyftdtd}
In general, both in the case of the simplest model of the Newtonian
gravity and in the case of any other alternative model of gravity,
we \underline{always} assume that the vectorial gravitational
potential $\vec v:=\vec v(\vec x,t)$ is an asymptotically acceptable
speed-like vector field (see Definition \ref{jhjkhhjhjhgjg;kkljkl}
and Proposition \ref{jhjkhhjhjhgjgjjkjkjklljj}). Therefore,
Corollary \ref{jhjkhhjhjhgjgjjkjkjhhjhCC} implies that there exist a
uniquely defined \underline{generally trivial speed-like} vector
field $\,\vec k:=\vec k(\vec x,t)\in\mathbb{R}^3$ (see Definition
\ref{jhjkhhjhjhgjg}) and a uniquely defined \underline{proper}
vector field $\,\vec h:=\vec h(\vec x,t)\in\mathbb{R}^3$, so that in
every cartesian coordinate system we have
\begin{equation}
\label{jljjjjkhhhjhjhjhljjjljljouukhhjhjgCChhhh} \lim\limits_{|\vec
x|\to+\infty}\vec h(\vec x,t)=0\quad\quad\forall
t\quad\quad\text{and}\quad\quad\lim\limits_{|\vec
x|\to+\infty}\left(d_{\vec x}\vec h(\vec
x,t)\right)=0\quad\quad\forall t\,,
\end{equation}
and in every cartesian coordinate system we can decompose vector
field $\vec v$ as:
\begin{equation}
\label{jljjjjkhhhjhjhjhljjjljljkjkCChhhh} \vec v(\vec x,t)=\vec
h(\vec x,t)+\vec k(\vec x,t)\quad\quad\quad\quad\forall (\vec x,t).
\end{equation}
Then we call $\vec h$ \underline{vectorial potential of genuine
gravity} and $\vec k$ \underline{vectorial potential of inertia}. In
particular, it is clear that, given cartesian coordinate system
$(*)$, this system will be \underline{inertial} (in the frames of
Definition \ref{jggggggggggggggggggggjklj}) if and only if the
vectorial potential of inertia $\vec k$ is a \underline{constant}
(independent of $(\vec x,t)$) in the coordinate system $(*)$.
\end{definition}
The name of genuine gravity and inertia is clarified by the fact
that since $\vec k$ is \underline{generally trivial} speed-like
vector field, by the definition, we can completely eliminate $\vec
k$ (getting $\vec k=0$) by the simple change of cartesian coordinate
system (it is fictitious). On the other hand, if $\vec h$ is
nontrivial, since $\vec h$ is a \underline{proper} vector field, we
cannot make it trivial in any other coordinate system (it is
genuine). Moreover, the vector of inertia $\vec k$ depends only on
the coordinate system in the space and it is completely independent
on the physical matter or physical fields filling this space. In
contrast vectorial potential of genuine gravity $\vec h$ depends
essentially on the surrounding physical matter (in the model of the
Newtonian gravity through gravitational masses).

Next, assume the model of Newtonian gravity, given by
\er{MaxVacFull1ninshtrgravortghhghgjkgghklhjgkghghjjkjhjkkggjkhjkhjjhhfhjhklkhkhjjklzzzyyyNWBWHWPPN222ll},
and assume that $(*)$ is some inertial (in the frames of Definition
\ref{jggggggggggggggggggggjklj}) cartesian coordinate system. Then
by Proposition \ref{jghjgyfyffffhfgljjkjl} in the system $(*)$ we
have:
\begin{equation}
\label{MaxVacFull1ninshtrgravortghhghgjkgghklhjgkghghjjkjhjkkggjkhjkhjjhhfhjhklkhkhjjklzzzyyyhjggjhgghhjhNWNWBWHWPPN222hhhh}
\begin{cases}
curl_{\vec x}\vec v= 0,\\
\frac{\partial\vec v}{\partial t}+d_\vec x\vec v\cdot\vec v=
-\nabla_{\vec x}\Phi,
\end{cases}
\end{equation}
where $\Phi$ is a scalar Newtonian gravitational potential which is
assumed to be a proper scalar, which satisfies
\begin{equation}
\label{MaxVacFull1ninshtrgravortghhghgjkgghklhjgkghghjjkjhjkkggjkhjkhjjhhfhjhklkhkhjjklzzzyyyhjggjhgghhjhNWNWNWBWHWPPN222hhhh}
\begin{cases}
\Delta_{\vec x}\Phi=4\pi GM\quad\quad\quad\forall \,(\vec x,t),\\
\lim\limits_{|\vec x|\to+\infty}\nabla_{\vec x}\Phi(\vec
x,t)=0\quad\quad\quad\forall t,
\end{cases}
\end{equation}
with $M$ being the gravitational mass density and $G$ being the
gravitational constant. Thus inserting
\er{jljjjjkhhhjhjhjhljjjljljkjkCChhhh} into
\er{MaxVacFull1ninshtrgravortghhghgjkgghklhjgkghghjjkjhjkkggjkhjkhjjhhfhjhklkhkhjjklzzzyyyhjggjhgghhjhNWNWBWHWPPN222hhhh}
and using the fact that in the inertial system $(*)$ the vector
$\vec k$ is a constant, we deduce
\begin{equation}
\label{MaxVacFull1ninshtrgravortghhghgjkgghklhjgkghghjjkjhjkkggjkhjkhjjhhfhjhklkhkhjjklzzzyyyhjggjhgghhjhNWNWBWHWPPN222hhhhkjhkhk}
\begin{cases}
curl_{\vec x}\vec h= 0,\\
\frac{\partial\vec h}{\partial t}+d_\vec x\vec h\cdot(\vec k+\vec
h)= -\nabla_{\vec x}\Phi.
\end{cases}
\end{equation}
In particular, by
\er{MaxVacFull1ninshtrgravortghhghgjkgghklhjgkghghjjkjhjkkggjkhjkhjjhhfhjhklkhkhjjklzzzyyyhjggjhgghhjhNWNWBWHWPPN222hhhhkjhkhk},
using the fact $curl_{\vec x}\vec h=0$ we deduce:
\begin{equation}
\label{MaxVacFull1ninshtrgravortghhghgjkgghklhjgkghghjjkjhjkkggjkhjkhjjhhfhjhklkhkhjjklzzzyyyhjggjhgghhjhNWNWBWHWPPN222hhhhkjhkhkjkhh}
\begin{cases}
curl_{\vec x}\vec h= 0,\\
\frac{\partial\vec h}{\partial t}-(\vec k+\vec h)\times curl_{\vec
x}\vec h+\{d_\vec x\vec h\}^T\cdot(\vec k+\vec h)= -\nabla_{\vec
x}\Phi,
\end{cases}
\end{equation}
that we finally rewrite as:
\begin{equation}
\label{MaxVacFull1ninshtrgravortghhghgjkgghklhjgkghghjjkjhjkkggjkhjkhjjhhfhjhklkhkhjjklzzzyyyhjggjhgghhjhNWNWBWHWPPN222hhhhkjhkhkjkhhjkkhhkjkljkkhkh}
\begin{cases}
curl_{\vec x}\vec h= 0,\\
\frac{\partial\vec h}{\partial t}-\vec v\times curl_{\vec x}\vec
h+\nabla_{\vec x}\left(\vec v\cdot\vec h\right)-\nabla_{\vec
x}\left(\frac{1}{2}\left|\vec h\right|^2\right)=-\nabla_{\vec
x}\Phi,
\end{cases}
\end{equation}
or alternatively as:
\begin{equation}
\label{MaxVacFull1ninshtrgravortghhghgjkgghklhjgkghghjjkjhjkkggjkhjkhjjhhfhjhklkhkhjjklzzzyyyhjggjhgghhjhNWNWBWHWPPN222hhhhkjhkhkjkhhjkkhhkjkljkkhkh11}
\begin{cases}
curl_{\vec x}\vec h= 0,\\
\frac{\partial\vec h}{\partial t}-\vec k\times curl_{\vec x}\vec
h+\nabla_{\vec x}\left(\vec k\cdot\vec h\right)+\nabla_{\vec
x}\left(\frac{1}{2}\left|\vec h\right|^2\right)=-\nabla_{\vec
x}\Phi.
\end{cases}
\end{equation}
However, since $\vec h$ is a proper vector field, by Proposition
\ref{yghgjtgyrtrt} we deduce, that both
\er{MaxVacFull1ninshtrgravortghhghgjkgghklhjgkghghjjkjhjkkggjkhjkhjjhhfhjhklkhkhjjklzzzyyyhjggjhgghhjhNWNWBWHWPPN222hhhhkjhkhkjkhhjkkhhkjkljkkhkh}
and
\er{MaxVacFull1ninshtrgravortghhghgjkgghklhjgkghghjjkjhjkkggjkhjkhjjhhfhjhklkhkhjjklzzzyyyhjggjhgghhjhNWNWBWHWPPN222hhhhkjhkhkjkhhjkkhhkjkljkkhkh11}
preserve their form also in \underline{non-inertial} cartesian
coordinate systems, provided $\Phi$ is a proper scalar, which is
defined in \underline{every} inertial or non-inertial cartesian
coordinate system by the following:
\begin{equation}
\label{MaxVacFull1ninshtrgravortghhghgjkgghklhjgkghghjjkjhjkkggjkhjkhjjhhfhjhklkhkhjjklzzzyyyhjggjhgghhjhNWNWNWBWHWPPN222hhhhjkhkjhjhjhjkl}
\begin{cases}
\Delta_{\vec x}\Phi=4\pi GM\quad\quad\quad\forall \,(\vec x,t),\\
\lim\limits_{|\vec x|\to+\infty}\nabla_{\vec x}\Phi(\vec
x,t)=0\quad\quad\quad\forall t.
\end{cases}
\end{equation}
Finally, by Proposition \ref{jhjkhhjhjhgjgjjkjk} vector field $\vec
k$ satisfies
\begin{equation}
\label{MaxVacFull1ninshtrgravortghhghgjkgghklhjgkghghjjkjhjkkggjkhjkhjjhhfhjhklkhkhjjklzzzyyyhjggjhgghhjhNWNWNWBWHWPPN222hhhhjkhkjhjhjhhkhhjhjkjhmjljkl}
d_{\vec x}\vec k+\left\{d_{\vec x}\vec
k\right\}^T=0\quad\quad\quad\forall \,(\vec x,t)\,,
\end{equation}
in \underline{every} inertial or non-inertial coordinate system.
Therefore the Newtonian gravity in the form
\er{MaxVacFull1ninshtrgravortghhghgjkgghklhjgkghghjjkjhjkkggjkhjkhjjhhfhjhklkhkhjjklzzzyyyNWBWHWPPN222ll},
can be rewritten in the terms of genuine gravity and inertia as
\er{MaxVacFull1ninshtrgravortghhghgjkgghklhjgkghghjjkjhjkkggjkhjkhjjhhfhjhklkhkhjjklzzzyyyhjggjhgghhjhNWNWNWBWHWPPN222hhhhjkhkjhjhjhjkl}
and either
\er{MaxVacFull1ninshtrgravortghhghgjkgghklhjgkghghjjkjhjkkggjkhjkhjjhhfhjhklkhkhjjklzzzyyyhjggjhgghhjhNWNWBWHWPPN222hhhhkjhkhkjkhhjkkhhkjkljkkhkh}
or, alternatively,
\er{MaxVacFull1ninshtrgravortghhghgjkgghklhjgkghghjjkjhjkkggjkhjkhjjhhfhjhklkhkhjjklzzzyyyhjggjhgghhjhNWNWBWHWPPN222hhhhkjhkhkjkhhjkkhhkjkljkkhkh11},
that are complemented with
\er{MaxVacFull1ninshtrgravortghhghgjkgghklhjgkghghjjkjhjkkggjkhjkhjjhhfhjhklkhkhjjklzzzyyyhjggjhgghhjhNWNWNWBWHWPPN222hhhhjkhkjhjhjhhkhhjhjkjhmjljkl}
and \er{jljjjjkhhhjhjhjhljjjljljkjkCChhhh}. Moreover, note again
that
\er{MaxVacFull1ninshtrgravortghhghgjkgghklhjgkghghjjkjhjkkggjkhjkhjjhhfhjhklkhkhjjklzzzyyyhjggjhgghhjhNWNWNWBWHWPPN222hhhhjkhkjhjhjhjkl},
\er{MaxVacFull1ninshtrgravortghhghgjkgghklhjgkghghjjkjhjkkggjkhjkhjjhhfhjhklkhkhjjklzzzyyyhjggjhgghhjhNWNWBWHWPPN222hhhhkjhkhkjkhhjkkhhkjkljkkhkh},
\er{MaxVacFull1ninshtrgravortghhghgjkgghklhjgkghghjjkjhjkkggjkhjkhjjhhfhjhklkhkhjjklzzzyyyhjggjhgghhjhNWNWBWHWPPN222hhhhkjhkhkjkhhjkkhhkjkljkkhkh11},
\er{MaxVacFull1ninshtrgravortghhghgjkgghklhjgkghghjjkjhjkkggjkhjkhjjhhfhjhklkhkhjjklzzzyyyhjggjhgghhjhNWNWNWBWHWPPN222hhhhjkhkjhjhjhhkhhjhjkjhmjljkl},
and \er{jljjjjkhhhjhjhjhljjjljljkjkCChhhh} are invariant under the
change of inertial or non-inertial cartesian coordinate system.
%
%
%
%
%
%

Next note that, since we have $curl_{\vec x}\vec h=0$, we can write
either
\er{MaxVacFull1ninshtrgravortghhghgjkgghklhjgkghghjjkjhjkkggjkhjkhjjhhfhjhklkhkhjjklzzzyyyhjggjhgghhjhNWNWBWHWPPN222hhhhkjhkhkjkhhjkkhhkjkljkkhkh}
or
\er{MaxVacFull1ninshtrgravortghhghgjkgghklhjgkghghjjkjhjkkggjkhjkhjjhhfhjhklkhkhjjklzzzyyyhjggjhgghhjhNWNWBWHWPPN222hhhhkjhkhkjkhhjkkhhkjkljkkhkh11}
as the following Hamilton-Jacobi type equation:
\begin{equation}
\label{MaxVacFull1ninshtrgravortghhghgjkgghklhjgkghghjjkjhjkkggjkhjkhjjhhfhjhklkhkhjjklzzzyyyhjggjhgghhjhNWNWNWNWNWBWHWPPN222jujhkhhjhgjgg}
\begin{cases}
\vec h=\nabla_{\vec x}Y,\\
\frac{\partial Y}{\partial t}+\vec v\cdot\nabla_{\vec
x}Y-\frac{1}{2}\left|\nabla_{\vec x}Y\right|^2=\frac{\partial
Y}{\partial t}+\vec k\cdot\nabla_{\vec
x}Y+\frac{1}{2}\left|\nabla_{\vec x}Y\right|^2=-\Phi,
\end{cases}
\end{equation}
where $Y:=Y(\vec x,t)$ is some scalar field. Then, since $\vec h$ is
a proper vector field and $\Phi$ is proper scalar field,  by
Proposition \ref{yghgjtgyrtrt} we deduce, that
\er{MaxVacFull1ninshtrgravortghhghgjkgghklhjgkghghjjkjhjkkggjkhjkhjjhhfhjhklkhkhjjklzzzyyyhjggjhgghhjhNWNWNWNWNWBWHWPPN222jujhkhhjhgjgg}
is invariant under the changes of inertial or non-inertial cartesian
coordinate systems, provided that we consider $Y$ to be a
\underline{proper} scalar field. Note here that, in contrast to the
fact that the quantity $Z$ and Hamilton-Jacobi equation
\er{MaxVacFull1ninshtrgravortghhghgjkgghklhjgkghghjjkjhjkkggjkhjkhjjhhfhjhklkhkhjjklzzzyyyhjggjhgghhjhNWNWNWNWNWBWHWPPN222}
were defined only in inertial coordinate systems, the quantity $Y$
and Hamilton-Jacobi equation
\er{MaxVacFull1ninshtrgravortghhghgjkgghklhjgkghghjjkjhjkkggjkhjkhjjhhfhjhklkhkhjjklzzzyyyhjggjhgghhjhNWNWNWNWNWBWHWPPN222jujhkhhjhgjgg}
are defined in every inertial or non-inertial coordinate system.
Next in every \underline{inertial} cartesian coordinate system
define the scalar field $X:=X(\vec x,t)$ by
\begin{equation}
\label{MaxVacFull1ninshtrgravortghhghgjkgghklhjgkghghjjkjhjkkggjkhjkhjjhhfhjhklkhkhjjklzzzyyyhjggjhgghhjhNWNWNWNWNWBWHWPPN222jujhkhhjhgjggljkyiyhkh}
X(\vec x,t)=Z(\vec x,t)-Y(\vec x,t)\quad\quad\forall(\vec x,t)\,.
\end{equation}
Then, by the first equation in
\er{MaxVacFull1ninshtrgravortghhghgjkgghklhjgkghghjjkjhjkkggjkhjkhjjhhfhjhklkhkhjjklzzzyyyhjggjhgghhjhNWNWNWNWNWBWHWPPN222}
and the first equation in
\er{MaxVacFull1ninshtrgravortghhghgjkgghklhjgkghghjjkjhjkkggjkhjkhjjhhfhjhklkhkhjjklzzzyyyhjggjhgghhjhNWNWNWNWNWBWHWPPN222jujhkhhjhgjgg},
together with \er{jljjjjkhhhjhjhjhljjjljljkjkCChhhh}, we deduce
\begin{equation}
\label{jljjjjkhhhjhjhjhljjjljljkjkCChhhhjjjkkh}
\vec k(\vec x,t)=
\nabla_{\vec x}X(\vec x,t)\quad\quad\quad\quad\forall (\vec x,t),
\end{equation}
in every inertial cartesian coordinate system. Moreover, since the
quantity $Y$ is a proper scalar and since by Remark \ref{ugyugg} the
quantity $Z$ is a speed-like field generator (see Definition
\ref{gjgjffjjjjjjjjjjjjjjjjjjjjjjhh2}) on the extended Galilean
group of all inertial cartesian coordinate systems, by
\er{MaxVacFull1ninshtrgravortghhghgjkgghklhjgkghghjjkjhjkkggjkhjkhjjhhfhjhklkhkhjjklzzzyyyhjggjhgghhjhNWNWNWNWNWBWHWPPN222jujhkhhjhgjggljkyiyhkh}
we deduce that the quantity $X$ in
\er{jljjjjkhhhjhjhjhljjjljljkjkCChhhhjjjkkh} is also a speed-like
field generator on the extended Galilean group of all inertial
cartesian coordinate systems. In particular, under the change of
coordinate system $\bf{(*)}$ to $\bf{(**)}$ given by the Galilean
Transformation
\begin{equation}\label{noninchgravortbstrjgghguittu1intmmnn}
\begin{cases}
\vec x'=\vec x+\vec wt,\\
t'=t,
\end{cases}
\end{equation}
the quantity $X$ transforms as:
\begin{equation}\label{noninchgravortbstrjgghguittu1intmmjhhjnn}
X'(\vec x',t')=X(\vec x,t)+\vec w\cdot\vec x+\frac{1}{2}|\vec
w|^2t\quad\text{(modulo unspecified additive constant, independent
on $(\vec x,t)$)}.
\end{equation}
Finally note that, since in every inertial cartesian system $\vec k$
is a constant independent of $(\vec x,t)$, by
\er{jljjjjkhhhjhjhjhljjjljljkjkCChhhhjjjkkh} we deduce that the
quantity $X(\vec x,t)$ is a linear function of $(\vec x,t)$.

\section{Maxwell equations revised}\label{MaxRevsPPN} We would like
to make the laws of Electrodynamics in the vacuum to be invariant
under the Galilean transformations. For this purpose we refer to the
analogy with the Maxwell equations in a medium. It is well known
that
the classical Maxwell equations in a medium have the form of
\begin{equation}\label{MaxVacPPN}
\begin{cases}
curl_{\vec x} \vec H\equiv \frac{4\pi}{c}\vec j+\frac{1}{c}\frac{\partial \vec D}{\partial t}
,\\
div_{\vec x} \vec D\equiv 4\pi\rho
,\\
curl_{\vec x} \vec E+\frac{1}{c}\frac{\partial \vec B}{\partial t}\equiv 0
,\\
div_{\vec x} \vec B\equiv 0
.
\end{cases}
\end{equation}
Here $\vec E$ is the electric field, $\vec B$ is the magnetic field,
$\vec D$ is the electric displacement field, $\vec H$ is the $\vec
H$-magnetic field, $\rho$ is the charge density, $\vec j$ is the
current density and $c$ is the universal constant,
called
speed of light.
%
%
%
It is assumed in the Classical Electrodynamics that for the vacuum
we always have $\vec D\equiv \vec E$ and $\vec H\equiv \vec B$.

 We assume that the Maxwell equations in the vacuum
have the usual form \er{MaxVacPPN}, as in any other medium,
however, similarly to the General Theory of Relativity we assume
that the electromagnetic field is influenced by the gravitational
field.
%
%
%
%
%
%
Then, we assume that for a given inertial coordinate system we have
$\vec D(\vec x,t)= \vec E(\vec x,t)$ and $\vec H(\vec x,t)=\vec
B(\vec x,t)$ for the vacuum only in the case where the vectorial
gravitational potential $\vec v(\vec x,t)$ on the point $\vec x$ at
the time $t$ equals to zero in the given coordinate system i.e.
\begin{equation}\label{IdspEthPPN}
\text{If}\;\;\vec v(\vec x,t)=0\;\;\text{for some}\;\;(\vec
x,t)\in\R^3\times\R\;\;\text{then}\;\;\vec D(\vec x,t)= \vec E(\vec
x,t)\;\;\text{and}\;\;\vec H(\vec x,t)=\vec B(\vec x,t),
\end{equation}
where $\vec v(\vec x,t)$ is the same as in
\er{noninchgravortbstrjgghguittu2gjgghhjhghjhjgghgghghghtytythvfghfgghjggghjgjh}.
In order to obtain the relations $\vec D\sim \vec E$ and $\vec H\sim
\vec B$ in the general case we assume that the equations
\er{MaxVacPPN} and the Lorentz force $\vec F:=\sigma \vec
E+\frac{\sigma}{c}\,\vec u\times \vec B$ (where $\sigma$ is the
charge of the test particle and $\vec u$ is its velocity) are
invariant under the Galilean Transformations:
\begin{equation}\label{Gal2PPN}
\begin{cases}
\vec x'=\vec x+t\vec w,\\
t'=t.
\end{cases}
\end{equation}
First observe that if
$\vec u$ is a velocity of the test particle then $\vec u'=\vec
u+\vec w$. Thus, since we assumed that the Lorentz force $\vec
F:=\sigma \vec E+\frac{\sigma}{c}\vec u\times \vec B$ is invariant
under Galilean transformation we infer $$\sigma \vec
E'+\frac{\sigma}{c}\, (\vec u+\vec w)\times \vec B'=\sigma \vec
E'+\frac{\sigma}{c}\, \vec u'\times \vec B'=\vec F'=\vec F=\sigma
\vec E+\frac{\sigma}{c}\, \vec u\times \vec B.$$ Therefore, we
obtain the following identities:
\begin{equation}\label{EBTransPPN}
\begin{cases}
\vec E'=\vec E-\frac{1}{c}\,\vec w\times \vec B,\\
\vec B'=\vec B.
\end{cases}
\end{equation}
It is easy to check that, under transformations \er{Gal2PPN} and
\er{EBTransPPN}, the last two equations in \er{MaxVacPPN} are
invariant. Next observe that in the absents of currents and charges
the first two equations in \er{MaxVacPPN} for $\vec H$ and $\vec D$
will be the same as the last two for $\vec E$ and $\vec B$ if we
will change the sign of the time there. Therefore, it can be assumed
that the first two equations will stay invariant under the
transformation:
\begin{equation}\label{HDTransPPN}
\begin{cases}
\vec H'=\vec H+\frac{1}{c}\,\vec w\times \vec D,\\
\vec D'=\vec D.
\end{cases}
\end{equation}
Indeed, since $\rho'=\rho$ and $\vec j'=\vec j+\rho \vec w$, it can
be easily checked that under the transformations \er{Gal2PPN} and
\er{HDTransPPN} the first two equations will stay invariant also in
the case of charges and currents. Therefore, we obtained that all
equations in \er{MaxVacPPN} are invariant under the transformations
\er{Gal2PPN}
%
%
%
%
%
%
and
\begin{equation}\label{EBDHTrans2PPN}
\begin{cases}
\vec D'=\vec D,\\
\vec B'=\vec B,\\
\vec E'=\vec E-\frac{1}{c}\,\vec w\times \vec B,\\
\vec H'=\vec H+\frac{1}{c}\,\vec w\times \vec D.
\end{cases}
\end{equation}
Next fix some point $(\vec x_1,t_1)\in\R^3\times\R$ and consider
$\vec w:=-\vec v(\vec x_1,t_1)$, where $\vec v$ is the vectorial
gravitational potential. Then, since $\vec v'=\vec v+\vec w$
(speed-like vector field), we obtain that at the point $(\vec
x'_1,t'_1)$ we have $\vec v'=0$. Therefore, by the assumption
\er{IdspEthPPN} we must have $\vec E'=\vec D'$ and $\vec H'=\vec B'$
at this point. Plugging it into \er{EBDHTrans2PPN}, for this point
we obtain
\begin{multline}
\label{EBDHTranssled1PPN} \vec E(\vec x_1,t_1)+\frac{\vec v(\vec
x_1,t_1)}{c}\times \vec B(\vec x_1,t_1)=\vec E(\vec
x_1,t_1)-\frac{\vec w}{c}\times \vec B(\vec x_1,t_1)\\=\vec E'(\vec
x'_1,t'_1)=\vec D'(\vec x'_1,t'_1)=\vec D(\vec x_1,t_1),
\end{multline}
and
\begin{multline} \label{EBDHTranssled2PPN}
\vec H(\vec x_1,t_1)-\frac{\vec v(\vec x_1,t_1)}{c}\times \vec
D(\vec x_1,t_1)=\vec H(\vec x_1,t_1)+\frac{\vec w}{c}\times \vec
D(\vec x_1,t_1)\\=\vec H'(\vec x'_1,t'_1)=\vec B'(\vec
x'_1,t'_1)=\vec B(\vec x_1,t_1).
\end{multline}
Thus, since the point $(\vec x_1,t_1)\in\R^3\times\R$ was
arbitrarily chosen, by \er{EBDHTranssled1PPN} and
\er{EBDHTranssled2PPN} we obtain the following relations
\begin{equation}\label{DtoEBtoHPPN}
\begin{cases}
\vec E(\vec x,t)=\vec D(\vec x,t)-\frac{1}{c}\,\vec v(\vec
x,t)\times
\vec B(\vec x,t)\quad\quad\forall(\vec x,t)\in\R^3\times\R\\
\vec H(\vec x,t)=\vec B(\vec x,t)+\frac{1}{c}\,\vec v(\vec
x,t)\times \vec D(\vec x,t)\quad\quad\forall(\vec
x,t)\in\R^3\times\R.
\end{cases}
\end{equation}
Plugging \er{DtoEBtoHPPN} into \er{MaxVacPPN} we obtain the full
system of Electrodynamics in the case of an arbitrarily vectorial
gravitational potential:
\begin{equation}\label{MaxVacFullPPN}
\begin{cases}
curl_{\vec x} \vec H\equiv \frac{4\pi}{c}\vec j+\frac{1}{c}\frac{\partial \vec D}{\partial t}
,\\
div_{\vec x} \vec D\equiv 4\pi\rho
,\\
curl_{\vec x} \vec E+\frac{1}{c}\frac{\partial \vec B}{\partial t}\equiv 0
,\\
div_{\vec x} \vec B\equiv
0
,\\
\vec E=\vec D-\frac{1}{c}\,\vec v\times
\vec B
,\\
\vec H=\vec B+\frac{1}{c}\,\vec v\times \vec
D
,
\end{cases}
\end{equation}
where $\vec v$ is the vectorial gravitational potential. It can be
easily checked that system \er{MaxVacFullPPN} and the Lorentz force
$\vec F:=\sigma(\vec E+\frac{\vec u}{c}\times \vec B)$ are invariant
under transformations \er{Gal2PPN} and \er{EBDHTrans2PPN}.
%
%
%
%
%
%
%
%
Note here that $\vec D$ and $\vec B$ are invariant under the change
of inertial coordinate system. Moreover, we can write the Lorentz
force as $\vec F:=\sigma(\vec D+\frac{\vec u-\vec v}{c}\times \vec
B)$, where $(\vec u-\vec v)$ is the relative velocity of the test
particle with respect to the fictitious aether.


\section{Maxwell equations in non-inertial cartesian coordinate
systems}\label{INNONredPPN}
Consider the change of certain non-inertial cartesian coordinate
system $(*)$ to another cartesian coordinate system $(**)$:
\begin{equation}\label{noninchredPPN}
\begin{cases}
\vec x'=A(t)\cdot \vec x+\vec z(t),\\
t'=t,
\end{cases}
\end{equation}
where $A(t)\in SO(3)$ is a rotation i.e. $A(t)\in \R^{3\times 3}$,
$det\, A(t)>0$ and $A(t)\cdot A^T(t)=I$ (here $A^T$ is the transpose
matrix of $A$ and $I$ is the identity matrix). Next, assume that in
coordinate system $(**)$ we observe a validity of Maxwell Equations
for the vacuum in the form:
\begin{equation}\label{MaxVacFull1ninshtrredPPN}
\begin{cases}
curl_{\vec x'} \vec H'\equiv \frac{4\pi}{c}\vec
j'+\frac{1}{c}\frac{\partial
\vec D'}{\partial t'},\\
div_{\vec x'} \vec D'\equiv 4\pi\rho',\\
curl_{\vec x'} \vec E'+\frac{1}{c}\frac{\partial \vec B'}{\partial t'}\equiv 0,\\
div_{\vec x'} \vec B'\equiv 0,\\
\vec E'=\vec D'-\frac{1}{c}\,\vec v'\times \vec B',\\
\vec H'=\vec B'+\frac{1}{c}\,\vec v'\times \vec D'.
\end{cases}
\end{equation}
Moreover, we assume that in coordinate system $(**)$ we observe a
validity of expression for the Lorentz force
\begin{equation}\label{LorenzChredPPN}
\vec F':=\sigma' \vec E'+\frac{\sigma'}{c}\,\vec u'\times \vec B'
\end{equation}
(where $\sigma'$ is the charge of the test particle and $\vec u'$ is
its velocity in coordinate system $(**)$). All above happens, in
particular, if coordinate system $(**)$ is inertial. Observe that if
$\vec F$ is the force in coordinate system $(*)$ which corresponds
to the Lorentz force $\vec F'$ in coordinate system $(**)$, then we
must have $\vec F'=A(t)\cdot\vec F$. Moreover, denoting $\vec
w(t)=\vec z'(t)$, we have the following obvious relations between
the physical characteristics in coordinate systems $(*)$ and $(**)$:
\begin{align}
\label{NoIn1redPPN}\vec F'=A(t)\cdot\vec F,\\
\label{NoIn2redPPN}\sigma'=\sigma,\\
\label{NoIn3redPPN}\vec u'=A(t)\cdot \vec u+A'(t)\cdot\vec x+\vec w(t),\\
\label{NoIn4redPPN}\rho'=\rho,\\
\label{NoIn5redPPN}\vec v'=A(t)\cdot \vec v+A'(t)\cdot\vec x+\vec w(t),\\
\label{NoIn6redPPN}\vec j'=A(t)\cdot \vec j+\rho A'(t)\cdot\vec
x+\rho\vec w(t)
\end{align}
(where $A'(t)$ is a derivative of $A(t)$). We consider the fields
$\vec E$ and $\vec B$ in the coordinate system $(*)$ to be defined
by the expression of Lorentz force:
\begin{equation}\label{LorenzChllredPPN}
\vec F=\sigma \vec E+\frac{\sigma}{c}\,\vec u\times \vec B.
\end{equation}
Plugging it into \er{LorenzChredPPN} and using \er{NoIn1redPPN},
\er{NoIn2redPPN} and \er{NoIn3redPPN} we deduce
\begin{multline}\label{buiguyfttyjredPPN}
\sigma \left(\vec E'+\frac{1}{c}\,\left(A'(t)\cdot\vec x+\vec
w(t)\right)\times \vec B'\right)+\frac{\sigma}{c}\,\left(A(t)\cdot
\vec u\right)\times \vec B'\\=\sigma \vec
E'+\frac{\sigma}{c}\,\left(A(t)\cdot \vec u+A'(t)\cdot\vec x+\vec
w(t)\right)\times \vec B'\\=\sigma' \vec E'+\frac{\sigma'}{c}\,\vec
u'\times \vec B'=\vec F'=A(t)\cdot\vec F=\sigma A(t)\cdot\vec
E+\frac{\sigma}{c}\,A(t)\cdot\left(\vec u\times \vec B\right).
\end{multline}
Thus using the trivial identity
\begin{equation}\label{cufuyggjjjfredPPN}
A\cdot\left(\vec a\times\vec b\right)=\left(A\cdot\vec
a\right)\times\left(A\cdot\vec b\right)\quad\forall\, \vec
a\in\R^3,\;\;\forall\, \vec b\in\R^3,\;\;\forall\, A\in SO(3),
\end{equation}
by \er{buiguyfttyjredPPN} we deduce
\begin{multline}\label{buiguyfttyjuiuihjredPPN}
\sigma \left(\vec E'+\frac{1}{c}\,\left(A'(t)\cdot\vec x+\vec
w(t)\right)\times \vec B'\right)+\frac{\sigma}{c}\,\left(A(t)\cdot
\vec u\right)\times \vec B'\\=\sigma A(t)\cdot\vec
E+\frac{\sigma}{c}\,\left(A(t)\cdot\vec u\right)\times
\left(A(t)\cdot\vec B\right).
\end{multline}
Therefore, since \er{buiguyfttyjuiuihjredPPN} must be valid for
arbitrary choices of $\vec u$ we deduce
\begin{equation*}
\begin{cases}
\vec B'=A(t)\cdot\vec B,
\\
\vec E'+\frac{1}{c}\,\left(A'(t)\cdot\vec x+\vec w(t)\right)\times
\vec B'=A(t)\cdot\vec E.
\end{cases}
\end{equation*}
Therefore,
\begin{equation*}
\vec E'=A(t)\cdot\vec E-\frac{1}{c}\,\left(A'(t)\cdot\vec x+\vec
w(t)\right)\times \vec B'=A(t)\cdot\vec
E-\frac{1}{c}\,\left(A'(t)\cdot\vec x+\vec w(t)\right)\times
\left(A(t)\cdot\vec B\right).
\end{equation*}
So we obtained the following relations linking the fields $\vec
E,\vec B$ in coordinate system $(*)$ and $\vec E',\vec B'$ in
coordinate system $(**)$:
\begin{equation}\label{yuythfgfyftydtydtydtyddredPPN}
\begin{cases}
\vec E'=A(t)\cdot\vec E-\frac{1}{c}\,\left(A'(t)\cdot\vec x+\vec
w(t)\right)\times \left(A(t)\cdot\vec B\right),
\\
\vec B'=A(t)\cdot\vec B.
\end{cases}
\end{equation}
Next, by \er{MaxVacFull1ninshtrredPPN} in coordinate system $(**)$
we have the relations
\begin{equation*}
\begin{cases}
\vec D'=\vec E'+\frac{1}{c}\,\vec v'\times \vec B',\\
\vec H'=\vec B'+\frac{1}{c}\,\vec v'\times \vec D'.
\end{cases}
\end{equation*}
Analogously we define $\vec D$ and $\vec H$  in coordinate system
$(*)$ by the formulas:
\begin{equation}\label{gjhghfhgdghdredPPN}
\begin{cases}
\vec D=\vec E+\frac{1}{c}\,\vec v\times \vec B,\\
\vec H=\vec B+\frac{1}{c}\,\vec v\times \vec D.
\end{cases}
\end{equation}
Then with the help of \er{yuythfgfyftydtydtydtyddredPPN},
\er{NoIn5redPPN} and \er{cufuyggjjjfredPPN} we deduce:
\begin{multline*}
\vec D'=\vec E'+\frac{1}{c}\,\vec v'\times \vec B'=\\A(t)\cdot\vec
E-\frac{1}{c}\,\left(A'(t)\cdot\vec x+\vec w(t)\right)\times
\left(A(t)\cdot\vec B\right)+\frac{1}{c}\,\vec v'\times
\left(A(t)\cdot\vec B\right)=\\A(t)\cdot\vec
E-\frac{1}{c}\,\left(A'(t)\cdot\vec x+\vec w(t)\right)\times
\left(A(t)\cdot\vec B\right)+\frac{1}{c}\,\left(A(t)\cdot \vec
v+A'(t)\cdot\vec x+\vec w(t)\right)\times \left(A(t)\cdot\vec
B\right)\\=A(t)\cdot\vec E+\frac{1}{c}\,\left(A(t)\cdot \vec
v\right)\times \left(A(t)\cdot\vec B\right)=A(t)\cdot\left(\vec
E+\frac{1}{c}\,\vec v\times \vec B\right)=A(t)\cdot \vec D,
\end{multline*}
and thus
\begin{multline*}
\vec H'=\vec B'+\frac{1}{c}\,\vec v'\times \vec D'=A(t)\cdot\vec
B+\frac{1}{c}\,\left(A(t)\cdot \vec v+A'(t)\cdot\vec x+\vec
w(t)\right)\times \left(A(t)\cdot\vec D\right)=\\
A(t)\cdot\vec B+\frac{1}{c}\,\left(A(t)\cdot \vec v\right)\times
\left(A(t)\cdot\vec D\right)+\frac{1}{c}\,\left(A'(t)\cdot\vec
x+\vec
w(t)\right)\times \left(A(t)\cdot\vec D\right)=\\
A(t)\cdot\left(\vec B+\frac{1}{c}\,\vec v\times \vec
D\right)+\frac{1}{c}\,\left(A'(t)\cdot\vec x+\vec w(t)\right)\times
\left(A(t)\cdot\vec D\right)\\=A(t)\cdot\vec
H+\frac{1}{c}\,\left(A'(t)\cdot\vec x+\vec w(t)\right)\times
\left(A(t)\cdot\vec D\right).
\end{multline*}
I.e. the following relations are valid:
\begin{equation}\label{yuythfgfyftydtydtydtyddyyyhhddhhhredPPN}
\begin{cases}
\vec D'=A(t)\cdot \vec D,\\
\vec B'=A(t)\cdot\vec B,\\
\vec E'=A(t)\cdot\vec E-\frac{1}{c}\,\left(A'(t)\cdot\vec x+\vec
w(t)\right)\times \left(A(t)\cdot\vec B\right),\\
\vec H'=A(t)\cdot\vec H+\frac{1}{c}\,\left(A'(t)\cdot\vec x+\vec
w(t)\right)\times \left(A(t)\cdot\vec D\right).
\end{cases}
\end{equation}
In particular vector fields $\vec D$ and $\vec B$ are proper vector
fields.

Next, by \er{noninchredPPN} and by Proposition \ref{yghgjtgyrtrt},
for every vector field $\vec \Gamma:\R^3\times[t_0,+\infty)\to\R^3$
we have
\begin{equation}
\label{vyguiuiujggghjjgredPPN}
\begin{cases}
d_{\vec x'}\vec \Gamma=\left(d_{\vec x}\vec \Gamma\right)\cdot A^{-1}(t)\\
curl_{\vec x'}\left( A(t)\cdot\vec \Gamma\right)=A(t)\cdot
curl_{\vec x}
\vec \Gamma\\
div_{\vec x'}\left( A(t)\cdot\vec \Gamma\right)=div_{\vec x}\vec
\Gamma.
\end{cases}
\end{equation}
Furthermore, by Proposition \ref{yghgjtgyrtrt},
for every vector field $\vec \Gamma:\R^3\times[t_0,+\infty)\to\R^3$
we have
\begin{multline}\label{bgvfgfhhtgjtggiuguiuiredPPN}
\frac{\partial \left(A(t)\cdot\vec \Gamma\right)}{\partial
t'}-curl_{\vec x'}\left(\vec v'\times\left(A(t)\cdot\vec
\Gamma\right)\right)+\left({div}_{\vec x'}\left(A(t)\cdot\vec
\Gamma\right)\right)\vec v'\\=A(t)\cdot\left(\frac{\partial \vec
\Gamma}{\partial t}- curl_{\vec x}\left(\vec v\times
\vec\Gamma\right)+\left({div}_{\vec x}\vec \Gamma\right)\vec
v\right).
\end{multline}
On the other hand, by \er{MaxVacFull1ninshtrredPPN} we have
\begin{multline}\label{uiguihjkjkklklklredPPN}
curl_{\vec x'} \vec B'- \frac{4\pi}{c}\left(\vec j'-\rho'\vec
v'\right)-\frac{1}{c}\left(\frac{\partial \vec D'}{\partial
t'}-curl_{\vec x'}\left(\vec v'\times\vec
D'\right)+\left({div}_{\vec x'}\vec D'\right)\vec
v'\right)\\=curl_{\vec x'} \vec H'- \frac{4\pi}{c}\vec
j'-\frac{1}{c}\frac{\partial \vec D'}{\partial t'}=0\end{multline}
and
\begin{equation}\label{uiguihjkjkklklkljhhredPPN}
curl_{\vec x'} \vec D'+\frac{1}{c}\left(\frac{\partial \vec
B'}{\partial t'}-curl_{\vec x'}\left(\vec v'\times\vec
B'\right)+\left({div}_{\vec x'}\vec B'\right)\vec
v'\right)=curl_{\vec x'} \vec E'+\frac{1}{c}\frac{\partial \vec
B'}{\partial t'}=0.
\end{equation}
Thus plugging \er{uiguihjkjkklklklredPPN} and
\er{uiguihjkjkklklkljhhredPPN} into \er{bgvfgfhhtgjtggiuguiuiredPPN}
and using \er{gjhghfhgdghdredPPN}, \er{NoIn4redPPN},
\er{NoIn5redPPN}, \er{NoIn6redPPN} and \er{vyguiuiujggghjjgredPPN}
gives
\begin{multline}\label{uiguihjkjkklklklhkkredPPN}
A(t)\cdot\left(curl_{\vec x} \vec H-\frac{4\pi}{c}\,\vec
j-\frac{1}{c}\frac{\partial \vec D}{\partial
t}+\frac{1}{c}\left(4\pi\rho-{div}_{\vec x}\vec
D\right)\vec v\right)=\\
A(t)\cdot\left(curl_{\vec x} \vec B- \frac{4\pi}{c}\left(\vec
j-\rho\vec v\right)-\frac{1}{c}\left(\frac{\partial \vec D}{\partial
t}-curl_{\vec x}\left(\vec v\times\vec D\right)+\left({div}_{\vec
x}\vec D\right)\vec v\right)\right)=\\curl_{\vec x'} \vec B'-
\frac{4\pi}{c}\left(\vec j'-\rho'\vec
v'\right)-\frac{1}{c}\left(\frac{\partial \vec D'}{\partial
t'}-curl_{\vec x'}\left(\vec v'\times\vec
D'\right)+\left({div}_{\vec x'}\vec D\right)\vec v'\right)=0.
\end{multline}
Similarly
\begin{multline}\label{uiguihjkjkklklkljhhjkggkjredPPN}
A(t)\cdot\left(curl_{\vec x} \vec E+\frac{1}{c}\frac{\partial \vec
B}{\partial t}+\frac{1}{c}\left({div}_{\vec x}\vec B\right)\vec
v\right)=\\A(t)\cdot\left(curl_{\vec x} \vec
D+\frac{1}{c}\left(\frac{\partial \vec B}{\partial t}-curl_{\vec
x}\left(\vec v\times\vec B\right)+\left({div}_{\vec x}\vec
B\right)\vec v\right)\right)\\= curl_{\vec x'} \vec
D'+\frac{1}{c}\left(\frac{\partial \vec B'}{\partial t'}-curl_{\vec
x'}\left(\vec v'\times\vec B'\right)+\left({div}_{\vec x'}\vec
B\right)\vec v'\right)=0.
\end{multline}
On the other hand, by \er{yuythfgfyftydtydtydtyddyyyhhddhhhredPPN},
\er{MaxVacFull1ninshtrredPPN}, \er{vyguiuiujggghjjgredPPN} and
\er{NoIn4redPPN} we obtain:
\begin{equation}
\label{vyguiuiujggkkkhjhjjppopredPPN} 4\pi\rho=4\pi\rho'=div_{\vec
x'} \vec D'=div_{\vec x} \vec D\quad\text{and}\quad 0=div_{\vec x'}
\vec B'=div_{\vec x} \vec B.
\end{equation}
Thus plugging \er{uiguihjkjkklklklhkkredPPN},
\er{uiguihjkjkklklkljhhjkggkjredPPN} and
\er{vyguiuiujggkkkhjhjjppopredPPN} we obtain
\begin{equation}\label{fhgjhkhkredPPN}
\begin{cases}
curl_{\vec x} \vec H= \frac{4\pi}{c}\vec
j+\frac{1}{c}\frac{\partial \vec D}{\partial t},\\
div_{\vec x} \vec D=4\pi\rho,
\\
curl_{\vec x} \vec E+\frac{1}{c}\frac{\partial \vec B}{\partial
t}=0,
\\
div_{\vec x} \vec B=0.
\end{cases}
\end{equation}
Then, plugging \er{fhgjhkhkredPPN} into
\er{gjhghfhgdghdredPPN}, we finally obtain that in coordinate system
$(*)$ the Maxwell equations have the same form as in system $(**)$
i.e.
\begin{equation}\label{MaxVacFull1ninshtrhjkkredPPN}
\begin{cases}
curl_{\vec x} \vec H\equiv \frac{4\pi}{c}\vec
j+\frac{1}{c}\frac{\partial
\vec D}{\partial t},\\
div_{\vec x} \vec D\equiv 4\pi\rho,\\
curl_{\vec x} \vec E+\frac{1}{c}\frac{\partial \vec B}{\partial t}\equiv 0,\\
div_{\vec x} \vec B\equiv 0,\\
\vec E=\vec D-\frac{1}{c}\,\vec v\times \vec B,\\
\vec H=\vec B+\frac{1}{c}\,\vec v\times \vec D.
\end{cases}
\end{equation}
Therefore, since the assumption, that coordinate system $(**)$ is
inertial, implies the relations of \er{MaxVacFull1ninshtrredPPN}, we
deduce that the expressions of Maxwell equations in the form
\er{MaxVacFull1ninshtrhjkkredPPN} and of the Lorentz force in the
form \er{LorenzChllredPPN}
%
%
%
%
%
%
are valid in every non-inertial cartesian coordinate system.
Moreover, under the change of the
system, given by \er{noninchredPPN}, the transformations of the
electromagnetic fields are given by
\er{yuythfgfyftydtydtydtyddyyyhhddhhhredPPN} i.e.
\begin{equation}\label{yuythfgfyftydtydtydtyddyyyhhddhhhredPPN111hgghjg}
\begin{cases}
\vec D'=A(t)\cdot \vec D,\\
\vec B'=A(t)\cdot\vec B,\\
\vec E'=A(t)\cdot\vec E-\frac{1}{c}\,\left(\frac{dA}{dt}(t)\cdot\vec
x+\frac{d\vec z}{dt}(t)\right)\times \left(A(t)\cdot\vec B\right),\\
\vec H'=A(t)\cdot\vec H+\frac{1}{c}\,\left(\frac{dA}{dt}(t)\cdot\vec
x+\frac{d\vec z}{dt}(t)\right)\times \left(A(t)\cdot\vec D\right).
\end{cases}
\end{equation}
%
%
%
%
%
%
%
%

So the laws of Electrodynamics are also invariant in non-inertial
coordinate systems.

\begin{remark}\label{gyugfjyhgghg}
Since as already mentioned before, the direction of the local
vectorial gravitational potential is normal to the Earth surface, in
the frames of our model, we provide a non-relativistic explanation
of the classical Michelson-Morley experiment. Indeed in this
experiment the axes of the apparatus are tangent to the Earth
surface and thus the null result cannot be affected by the vectorial
gravitational potential. Since, the value of the local vectorial
gravitational potential equals to the escape velocity, if we
consider the vertical Michelson-Morley experiment, where one of the
axes of the apparatus is normal to the Earth surface, then in the
frames of our model the expected result should be analogous to the
positive result of Aether drift with the speed equal to the escape
velocity. However, regarding the vertical Michelson-Morley
experiment i.e. the modification of Michelson-Morley experiment,
where at least one of the axes of the apparatus is not tangent to
the Earth surface, we found only very scarce and contradictory
information.
\end{remark}

\section{Scalar and vectorial electromagnetic potentials}\label{ghfvdgfdjfg}
Consider
the system of Maxwell equations in the vacuum of the form
\begin{equation}\label{MaxVacFull1bjkgjhjhgjaaaPPN}
\begin{cases}
curl_{\vec x} \vec H\equiv \frac{4\pi}{c}\vec
j+\frac{1}{c}\frac{\partial
\vec D}{\partial t},\\
div_{\vec x} \vec D\equiv 4\pi\rho,\\
curl_{\vec x} \vec E+\frac{1}{c}\frac{\partial \vec B}{\partial t}\equiv 0,\\
div_{\vec x} \vec B\equiv 0,\\
\vec E=\vec D-\frac{1}{c}\,\vec v\times \vec B,\\
\vec H=\vec B+\frac{1}{c}\,\vec v\times \vec D,
\end{cases}
\end{equation}
where $\vec v$ is the vectorial gravitational potential. Then by the
third and the fourth equations in \er{MaxVacFull1bjkgjhjhgjaaaPPN}
we can write:
\begin{equation}\label{MaxVacFull1bjkgjhjhgjgjgkjfhjfdghghligioiuittrPPN}
\begin{cases}
\vec B\equiv curl_{\vec x} \vec A,\\
\vec E\equiv-\nabla_{\vec x}\Psi-\frac{1}{c}\frac{\partial\vec
A}{\partial t},
\end{cases}
\end{equation}
where we call $\Psi$ and $\vec A$ the scalar and the vectorial
electromagnetic potentials. Then by
\er{MaxVacFull1bjkgjhjhgjgjgkjfhjfdghghligioiuittrPPN} and
\er{MaxVacFull1bjkgjhjhgjaaaPPN} we have
\begin{equation}\label{vhfffngghPPN333yuyu}
\begin{cases}
\vec B= curl_{\vec x} \vec A\\
\vec E=-\nabla_{\vec x}\Psi-\frac{1}{c}\frac{\partial\vec
A}{\partial t}\\
 \vec D=-\nabla_{\vec
x}\Psi-\frac{1}{c}\frac{\partial\vec A}{\partial t}+\frac{1}{c}\vec
v\times curl_{\vec x}\vec A\\
\vec H\equiv curl_{\vec x} \vec A+\frac{1}{c}\,\vec
v\times\left(-\nabla_{\vec x}\Psi-\frac{1}{c}\frac{\partial\vec
A}{\partial t}+\frac{1}{c}\vec v\times curl_{\vec x}\vec A\right).
\end{cases}
\end{equation}
We also define the proper scalar electromagnetic potential
$\Psi_0=\Psi_0(\vec x,t)$ by
\begin{equation}\label{vhfffngghhjghhgPPNghghghutghffugghjhjkjjkl}
\Psi_0:=\Psi-\frac{1}{c}\vec A\cdot\vec v.
\end{equation}
The name "proper" will be clarified bellow. Then, by
\er{vhfffngghPPN333yuyu} and
\er{vhfffngghhjghhgPPNghghghutghffugghjhjkjjkl} we have
\begin{equation}\label{vhfffngghPPN}
\begin{cases}
\vec B= curl_{\vec x} \vec A\\
\vec E=-\nabla_{\vec x}\Psi_0-\frac{1}{c}\frac{\partial\vec
A}{\partial t}-\frac{1}{c}\nabla_{\vec x}\left(\vec A\cdot\vec
v\right)
\\
 \vec D=-\nabla_{\vec
x}\Psi_0-\frac{1}{c}\left(\frac{\partial\vec A}{\partial t}-\vec
v\times curl_{\vec x}\vec A+\nabla_{\vec x}\left(\vec A\cdot\vec
v\right)\right)
\\
\vec H\equiv curl_{\vec x} \vec A-\frac{1}{c}\,\vec v\times
\left(-\nabla_{\vec x}\Psi_0+\frac{1}{c}\left(\frac{\partial\vec
A}{\partial t}-\vec v\times curl_{\vec x}\vec A+\nabla_{\vec
x}\left(\vec A\cdot\vec v\right)\right)\right).
\end{cases}
\end{equation}
The electromagnetic potentials are not uniquely defined and thus we
need to choose a calibration. For definiteness we can take $\vec A$
to satisfy
\begin{equation}\label{MaxVacFull1bjkgjhjhgjgjgkjfhjfdghghligioiuittrPPN22}
div_{\vec x}\vec A\equiv 0.
\end{equation}
%
%
%
%
%
%
It is clear that if $(\tilde\Psi,\tilde\Psi_0,\tilde{\vec A})$ is
another choice of electromagnetic potentials with a different
calibration then there exists a scalar field $w:=w(\vec x,t)$ such
that we have
\begin{equation}\label{MaxVacFull1bjkgjhjhgjgjgkjfhjfdghghligioiuittrPPNhjkjhkjgghhjjhj}
\begin{cases}
\tilde\Psi=\Psi+\frac{1}{c}\frac{\partial w}{\partial t}\\
\tilde{\vec A}=\vec A-\nabla_{\vec x}w\\
\tilde\Psi_0=\Psi_0+\frac{1}{c}\left(\frac{\partial w}{\partial
t}+\vec v\cdot\nabla_{\vec x}w\right).
\end{cases}
\end{equation}

Next consider the change of certain non-inertial cartesian
coordinate system $(*)$ to another cartesian coordinate system
$(**)$:
\begin{equation}\label{noninchredPPN111}
\begin{cases}
\vec x'=A(t)\cdot \vec x+\vec z(t),\\
t'=t,
\end{cases}
\end{equation}
where $A(t)\in SO(3)$ is a rotation i.e. $A(t)\in \R^{3\times 3}$,
$det\, A(t)>0$ and $A(t)\cdot A^T(t)=I$ (here $A^T$ is the transpose
matrix of $A$ and $I$ is the identity matrix). We are going to
investigate, what are the transformations of $(\Psi,\Psi_0,\vec
A)\sim(\Psi',\Psi'_0,\vec A')$ under the change of coordinates,
given by \er{noninchredPPN111}. Since, by
\er{yuythfgfyftydtydtydtyddyyyhhddhhhredPPN} the following relations
are valid
\begin{equation}\label{yuythfgfyftydtydtydtyddyyyhhddhhhredPPN111}
\begin{cases}
\vec D'=A(t)\cdot \vec D,\\
\vec B'=A(t)\cdot\vec B,\\
\vec E'=A(t)\cdot\vec E-\frac{1}{c}\,\left(\frac{dA}{dt}(t)\cdot\vec
x+\frac{d\vec z}{dt}(t)\right)\times \left(A(t)\cdot\vec B\right),\\
\vec H'=A(t)\cdot\vec H+\frac{1}{c}\,\left(\frac{dA}{dt}(t)\cdot\vec
x+\frac{d\vec z}{dt}(t)\right)\times \left(A(t)\cdot\vec D\right),
\end{cases}
\end{equation}
by the second equality in
\er{yuythfgfyftydtydtydtyddyyyhhddhhhredPPN111}, the first equality
in \er{MaxVacFull1bjkgjhjhgjgjgkjfhjfdghghligioiuittrPPN} and
\er{MaxVacFull1bjkgjhjhgjgjgkjfhjfdghghligioiuittrPPN22} we deduce
\begin{equation}\label{yuythfgfyftydtydtydtyddyyyhhddhhhredPPN111222}
\vec A'=A(t)\cdot \vec A,
\end{equation}
i.e. if $\vec A$ satisfies calibration
\er{MaxVacFull1bjkgjhjhgjgjgkjfhjfdghghligioiuittrPPN22} then it is
a proper vector field. On the other hand, by \er{vhfffngghPPN} we
have
\begin{equation}\label{vhfffngghPPNjuiuio}
\nabla_{\vec x}\Psi_0= -\vec D-\frac{1}{c}\left(\frac{\partial\vec
A}{\partial t}-\vec v\times curl_{\vec x}\vec A+\nabla_{\vec
x}\left(\vec A\cdot\vec v\right)\right).
\end{equation}
Thus by \er{yuythfgfyftydtydtydtyddyyyhhddhhhredPPN111222} and
\er{yuythfgfyftydtydtydtyddyyyhhddhhhredPPN111}, using Proposition
\ref{yghgjtgyrtrt}
we deduce that $\nabla_{\vec
x}\Psi_0$ is a proper vector field, i.e.
\begin{equation}\label{vhfffngghhjghhgPPNghghghuigigg}
\nabla_{\vec x'}\Psi'_0=A(t)\cdot\nabla_{\vec x}\Psi_0.
\end{equation}
So
\begin{equation}\label{vhfffngghhjghhgPPNghghghutghffohjh}
\Psi'_0=\Psi_0,
\end{equation}
i.e. $\Psi_0$ is a is proper scalar field, invariant under the
change of non-inertial cartesian coordinate systems. This explains
why we called $\Psi_0$ the proper scalar electromagnetic potential.
Then by \er{vhfffngghhjghhgPPNghghghutghffohjh} and
\er{vhfffngghhjghhgPPNghghghutghffugghjhjkjjkl} we deduce
\begin{equation}\label{vhfffngghhjghhgPPNghghghutghffghhg}
\left(\frac{1}{c}\vec A'\cdot\vec
v'-\Psi'\right)=\left(\frac{1}{c}\vec A\cdot\vec v-\Psi\right).
\end{equation}
Therefore, by \er{vhfffngghhjghhgPPNghghghutghffghhg},
\er{yuythfgfyftydtydtydtyddyyyhhddhhhredPPN111222}  and the fact
that
\begin{equation}
\label{NoIn5redPPNghjg}\vec v'=A(t)\cdot \vec
v+\frac{dA}{dt}(t)\cdot\vec x+\frac{d\vec z}{dt}(t),
\end{equation}
we deduce
\begin{equation}\label{vhfffngghhjghhgPPNghghghutghfflklhjkj}
\frac{1}{c}\vec A\cdot\left(\vec
v+A^T(t)\cdot\frac{dA}{dt}(t)\cdot\vec x+A^T(t)\cdot\frac{d\vec
z}{dt}(t)\right)-\Psi'=\frac{1}{c}\vec A\cdot\vec v-\Psi.
\end{equation}
So
\begin{equation}\label{vhfffngghhjghhgPPNghghghutghfflklhjkjhjg}
\Psi'=\Psi+\frac{1}{c}\vec
A\cdot\left(A^T(t)\cdot\frac{dA}{dt}(t)\cdot\vec
x+A^T(t)\cdot\frac{d\vec
z}{dt}(t)\right)=\Psi+\frac{1}{c}\left(A(t)\cdot\vec
A\right)\cdot\left(\frac{dA}{dt}(t)\cdot\vec x+\frac{d\vec
z}{dt}(t)\right).
\end{equation}
Therefore, under the change of some non-inertial cartesian
coordinate system $(*)$ to another cartesian coordinate system
$(**)$, given by \er{noninchredPPN111}, the electromagnetic
potentials transform as:
\begin{equation}\label{vhfffngghhjghhgPPNghghghutghfflklhjkjhjhjjgjkghhj}
\begin{cases}
\Psi'=
\Psi+\frac{1}{c}\left(\frac{dA}{dt}(t)\cdot\vec x+\frac{d\vec
z}{dt}(t)\right)\cdot\left(A(t)\cdot\vec A\right)
\\
\vec A'=A(t)\cdot \vec A\\
\Psi'_0:=\left(\Psi'-\frac{1}{c}\vec A'\cdot\vec
v'\right)=\Psi_0:=\left(\Psi-\frac{1}{c}\vec A\cdot\vec v\right).
\end{cases}
\end{equation}
In particular, under the Galilean transformations
\er{noninchgravortbstrjgghguittu1} the electromagnetic potentials
transform as:
%
%
%
%
%
%
\begin{equation}\label{vhfffngghhjghhgPPNghghghutghfflklhjkjhjhjjgjkghhjhhhjhgjgu}
\begin{cases}
\Psi'=
\Psi+\frac{1}{c}\vec w\cdot\vec A
\\
\vec A'=\vec A\\
\Psi'_0=\Psi_0.
\end{cases}
\end{equation}
In the proof of
\er{vhfffngghhjghhgPPNghghghutghfflklhjkjhjhjjgjkghhj} we used
equality \er{MaxVacFull1bjkgjhjhgjgjgkjfhjfdghghligioiuittrPPN22}
only for proof of equality
\er{yuythfgfyftydtydtydtyddyyyhhddhhhredPPN111222}. Thus relations
\er{vhfffngghhjghhgPPNghghghutghfflklhjkjhjhjjgjkghhj} are still
valid for every choice of calibration of $(\Psi,\Psi_0,\vec A)$,
which implies \er{yuythfgfyftydtydtydtyddyyyhhddhhhredPPN111222}. In
particular if $w$ is a proper scalar field i.e. $w'=w$ and if
$(\tilde\Psi,\tilde\Psi_0,\tilde{\vec A})$ is another choice of
electromagnetic potentials defined by
\begin{equation}\label{MaxVacFull1bjkgjhjhgjgjgkjfhjfdghghligioiuittrPPNhjkjhkjgghhjjhjkjggh}
\begin{cases}
\tilde\Psi=\Psi+\frac{1}{c}\frac{\partial w}{\partial t}\\
\tilde{\vec A}=\vec A-\nabla_{\vec x}w\\
\tilde\Psi_0=\Psi_0+\frac{1}{c}\left(\frac{\partial w}{\partial
t}+\vec v\cdot\nabla_{\vec x}w\right),
\end{cases}
\end{equation}
then, by Proposition \ref{yghgjtgyrtrt}
we have
\begin{equation}\label{vhfffngghhjghhgPPNghghghutghfflklhjkjhjhjjgjkghhjhjkhjg}
\begin{cases}
\tilde\Psi'= \tilde\Psi+\frac{1}{c}\left(\frac{dA}{dt}(t)\cdot\vec
x+\frac{d\vec z}{dt}(t)\right)\cdot\left(A(t)\cdot\tilde{\vec
A}\right)
\\
\tilde{\vec A}'=A(t)\cdot \tilde{\vec A}\\
\tilde\Psi'_0=\tilde\Psi_0.
\end{cases}
\end{equation}
On the other hand, we always can find a proper scalar field $w$ for
calibration to illuminate $\tilde\Psi_0$ in
\er{MaxVacFull1bjkgjhjhgjgjgkjfhjfdghghligioiuittrPPNhjkjhkjgghhjjhjkjggh}.
Then we have $\tilde\Psi_0\equiv 0$ and the electromagnetic fields
are fully represented by the vectorial electromagnetic potential
$\tilde{\vec A}$ analogously as the vectorial gravitational
potential represents the gravitational field. For this case, we
rewrite \er{vhfffngghPPN} as
\begin{equation}\label{vhfffngghPPNhjkhj}
\begin{cases}
\tilde\Psi_0=0\\
-\frac{1}{c}div_{\vec x}\left\{\frac{\partial\tilde{\vec
A}}{\partial t}-\vec v\times curl_{\vec x}\tilde{\vec
A}+\nabla_{\vec x}\left(\tilde{\vec A}\cdot\vec
v\right)\right\}=4\pi \rho
\\
\vec B= curl_{\vec x} \tilde{\vec A}\\
\vec E=-\frac{1}{c}\frac{\partial\tilde{\vec A}}{\partial
t}-\frac{1}{c}\nabla_{\vec x}\left(\tilde{\vec A}\cdot\vec v\right)
\\
 \vec D=-\frac{1}{c}\left(\frac{\partial\tilde{\vec A}}{\partial
t}-\vec v\times curl_{\vec x}\tilde{\vec A}+\nabla_{\vec
x}\left(\tilde{\vec A}\cdot\vec v\right)\right)
\\
\vec H\equiv curl_{\vec x} \tilde{\vec A}-\frac{1}{c}\,\vec v\times
\left(\frac{1}{c}\left(\frac{\partial\tilde{\vec A}}{\partial
t}-\vec v\times curl_{\vec x}\tilde{\vec A}+\nabla_{\vec
x}\left(\tilde{\vec A}\cdot\vec v\right)\right)\right).
\end{cases}
\end{equation}
Moreover, in this case
\er{vhfffngghhjghhgPPNghghghutghfflklhjkjhjhjjgjkghhjhjkhjg} is
satisfied.
%
%
%
%
%
%
%
%
%

\section{Lagrangian of the Electromagnetic field}\label{bhjghjfghfg} We would like to
present a Lagrangian and a variational principle for the
electromagnetic field and to obtain the Maxwell equations in the
form \er{MaxVacFull1bjkgjhjhgjaaaPPN} from this principle. Given
known the charge distribution $\rho:=\rho(\vec x,t)$, the current
distribution $\vec j:=\vec j(\vec x,t)$ and the vectorial
gravitational potential $\vec v:=\vec v(\vec x,t)$, consider a
Lagrangian density $L_1$ defined by
\begin{multline}\label{vhfffngghkjgghPPN}
L_1\left(\vec A,\Psi,\vec
x,t\right):=\frac{1}{8\pi}\left|-\nabla_{\vec
x}\Psi-\frac{1}{c}\frac{\partial\vec A}{\partial t}+\frac{1}{c}\vec
v\times curl_{\vec x}\vec A\right|^2-\frac{1}{8\pi}\left|curl_{\vec
x}\vec A\right|^2-\left(\rho\Psi-\frac{1}{c}\vec A\cdot\vec
j\right).
\end{multline}
We investigate stationary points of the functional
\begin{equation}\label{btfffygtgyggyPPN}
J=\int_0^T\int_{\mathbb{R}^3}L_1\left(\vec A,\Psi,\vec
x,t\right)d\vec x dt.
\end{equation}
We denote
\begin{equation}\label{guigjgjffghPPN}
\begin{cases}
\vec D=-\nabla_{\vec x}\Psi-\frac{1}{c}\frac{\partial\vec
A}{\partial t}+\frac{1}{c}\vec
v\times curl_{\vec x}\vec A\\
\vec B=curl_{\vec x}\vec A
\\
\vec E=-\nabla_{\vec x}\Psi-\frac{1}{c}\frac{\partial\vec A}{\partial t}=\vec D-\frac{1}{c}\vec v\times\vec B\\
\vec H=curl_{\vec x}\vec A+\frac{1}{c}\vec
v\times\left(-\nabla_{\vec x}\Psi-\frac{1}{c}\frac{\partial\vec
A}{\partial t}+\frac{1}{c}\vec v\times curl_{\vec x}\vec
A\right)=\vec B+\frac{1}{c}\vec
v\times\vec D\\
\Psi_0:=\Psi-\frac{1}{c}\vec A\cdot\vec v.
\end{cases}
\end{equation}
So we can write:
\begin{multline}\label{vhfffngghkjgghPPNggjgjjkgj}
L_1\left(\vec A,\Psi,\vec x,t\right):=\frac{1}{8\pi}\left|\vec
D\right|^2-\frac{1}{8\pi}\left|\vec
B\right|^2-\left(\rho\Psi-\frac{1}{c}\vec A\cdot\vec
j\right)\\=\frac{1}{8\pi}\left|\vec
D\right|^2-\frac{1}{8\pi}\left|\vec
B\right|^2-\rho\Psi_0+\frac{1}{c}\vec A\cdot(\vec j-\rho\vec v),
\end{multline}
and by \er{guigjgjffghPPN} we have:
\begin{equation}\label{guigjgjffghjhkkgPPN}
\begin{cases}
curl_{\vec x}\vec E+\frac{1}{c}\frac{\partial\vec B}{\partial t}=0\\
div_{\vec x}\vec B=0.
\end{cases}
\end{equation}
Moreover by \er{vhfffngghkjgghPPN} and \er{apfrm3} we have
\begin{equation}\label{vhfffngghkjgghggtghjgfhjhjkghghPPN}
0=\frac{\delta L_1}{\delta \Psi}=\frac{1}{4\pi}div_{\vec x}\vec
D-\rho,
\end{equation}
and
\begin{equation}\label{vhfffngghkjgghggtghjgfhjhjkghghyuiuuPPN}
0=\frac{\delta L_1}{\delta \vec A}=\frac{1}{c}\vec j+\frac{1}{4\pi
c}\frac{\partial\vec D}{\partial t}-\frac{1}{4\pi}curl_{\vec x}\vec
B-\frac{1}{4\pi c}curl_{\vec x}\left(\vec v\times \vec
D\right)=\frac{1}{c}\vec j+\frac{1}{4\pi c}\frac{\partial\vec
D}{\partial t}-\frac{1}{4\pi}curl_{\vec x}\vec H.
\end{equation}
So by \er{vhfffngghkjgghggtghjgfhjhjkghghPPN},
\er{vhfffngghkjgghggtghjgfhjhjkghghyuiuuPPN}, \er{guigjgjffghPPN}
and \er{guigjgjffghjhkkgPPN} we obtain the Maxwell equations in the
form:
\begin{equation}\label{guigjgjffghguygjyfPPN}
\begin{cases}
curl_{\vec x}\vec H=\frac{4\pi}{c}\vec j+\frac{\partial\vec
D}{\partial
t}\\
div_{\vec x}\vec D=4\pi\rho\\
curl_{\vec x}\vec E+\frac{1}{c}\frac{\partial\vec B}{\partial t}=0\\
div_{\vec x}\vec B=0\\
\vec E=\vec D-\frac{1}{c}\vec v\times\vec B\\
\vec H=\vec B+\frac{1}{c}\vec v\times\vec D.
\end{cases}
\end{equation}
Note also that, using \er{vhfffngghkjgghPPNggjgjjkgj}, by
\er{vhfffngghhjghhgPPNghghghutghfflklhjkjhjhjjgjkghhj} and
\er{yuythfgfyftydtydtydtyddyyyhhddhhhredPPN111} the Lagrangian $L_1$
is invariant, under the change of inertial or non-inertial
coordinate system, given by \er{noninchredPPN111}, i.e. for this
change we have
\begin{equation}\label{vhfffngghkjgghPPNggjgjjkgjjhhjkghg}
L'_1\left(\vec A',\Psi',\vec x',t'\right)=L_1\left(\vec A,\Psi,\vec
x,t\right).
\end{equation}

\subsection{Alternative Lagrangian of the Electromagnetic
field}\label{bhjghjfghfgjkjkjkhjh} Next we can associate an
alternative Lagrangian density related to electromagnetic field.
Given known the charge distribution $\rho:=\rho(\vec x,t)$, the
current distribution $\vec j:=\vec j(\vec x,t)$ and the vectorial
gravitational potential $\vec v:=\vec v(\vec x,t)$, again consider
$\vec D, \vec B, \vec E,\vec H$ be as in \er{guigjgjffghPPN}.
Then, by the elementary calculus there exist proper vector fields
$\vec A,\vec C$ and a proper scalar field $\psi_0$ such that
\begin{equation}\label{guigjgjffghPPNguyytyuyuuuuioio}
\begin{cases}
\vec D=curl_{\vec x}\vec C-\nabla_{\vec x}\psi_0\\
\vec B=curl_{\vec x}\vec A \,,
\end{cases}
\end{equation}
so that together with the identities
\begin{equation}\label{guigjgjffghguygjyfPPNyuuyuyyuyuuyuyiii}
\begin{cases}
\vec E=\vec D-\frac{1}{c}\vec v\times\vec B\\
\vec H=\vec B+\frac{1}{c}\vec v\times\vec D,
\end{cases}
\end{equation}
by \er{guigjgjffghPPNguyytyuyuuuuioio} we deduce:
\begin{equation}\label{guigjgjffghPPNguyytyuyuuuu}
\begin{cases}
\vec D=curl_{\vec x}\vec C-\nabla_{\vec x}\psi_0\\
\vec B=curl_{\vec x}\vec A
\\
\vec E=\left(curl_{\vec x}\vec C-\nabla_{\vec
x}\psi_0\right)-\frac{1}{c}\vec v\times\left(curl_{\vec x}\vec
A\right)
\\
\vec H=curl_{\vec x}\vec A+\frac{1}{c}\vec v\times\left(curl_{\vec
x}\vec C-\nabla_{\vec x}\psi_0\right)\,.
\end{cases}
\end{equation}
Then Maxwell equations \er{guigjgjffghguygjyfPPN} in the form:
\begin{equation}\label{guigjgjffghguygjyfPPNyuuyuyyuyuuyuyjjij}
\begin{cases}
curl_{\vec x}\vec H=\frac{4\pi}{c}\vec
j+\frac{1}{c}\frac{\partial\vec D}{\partial
t}\\
\Div_{\vec x}\vec D=4\pi\rho\\
curl_{\vec x}\vec E+\frac{1}{c}\frac{\partial\vec B}{\partial t}=0\\
\Div_{\vec x}\vec B=0\\
\vec E=\vec D-\frac{1}{c}\vec v\times\vec B\\
\vec H=\vec B+\frac{1}{c}\vec v\times\vec D,
\end{cases}
\end{equation}
are equivalent to
\begin{equation}\label{guigjgjffghguygjyfPPNyuuyuyyuyuuyuyjjijhhih333}
\begin{cases}
curl_{\vec x}\left\{curl_{\vec x}\vec A+\frac{1}{c}\vec
v\times\left(curl_{\vec x}\vec C-\nabla_{\vec
x}\psi_0\right)\right\}=\frac{4\pi}{c}\vec
j+\frac{1}{c}\frac{\partial}{\partial t}\left\{curl_{\vec x}\vec
C-\nabla_{\vec
x}\psi_0\right\}\\
-\Delta_{\vec x}\psi_0=4\pi\rho\\
curl_{\vec x}\left\{\left(curl_{\vec x}\vec C-\nabla_{\vec
x}\psi_0\right)-\frac{1}{c}\vec v\times\left(curl_{\vec x}\vec
A\right)\right\}+\frac{1}{c}\frac{\partial}{\partial
t}\left(curl_{\vec x}\vec A\right)=0.
\end{cases}
\end{equation}
that we rewrite as:
\begin{equation}\label{guigjgjffghguygjyfPPNyuuyuyyuyuuyuyjjijhhih}
\begin{cases}
-\Delta_{\vec x}\psi_0=4\pi\rho\\
curl_{\vec x}\left\{curl_{\vec x}\vec A\right\}-\frac{4\pi}{c}\vec
{\widetilde j}
=\frac{1}{c}\frac{\partial}{\partial t}\left\{curl_{\vec x}\vec
C\right\}-curl_{\vec x}\left\{\frac{1}{c}\vec
v\times\left(curl_{\vec x}\vec C\right)\right\}
\\
curl_{\vec x}\left(curl_{\vec x}\vec
C\right)=-\frac{1}{c}\frac{\partial}{\partial t}\left(curl_{\vec
x}\vec A\right)+curl_{\vec x}\left\{\frac{1}{c}\vec
v\times\left(curl_{\vec x}\vec A\right)\right\}.
\end{cases}
\end{equation}
where we set the reduced current:
\begin{equation}\label{reducedcurrentfhfhjfhjGGaauuui111}
\begin{cases}
\vec {\widetilde j}:=\vec j-\frac{1}{4\pi}\frac{\partial}{\partial
t} \left(\nabla_{\vec x}\psi_0\right)+\frac{1}{4\pi}curl_{\vec
x}\left(\vec v\times \nabla_{\vec x}\psi_0\right),\\
-\Delta_{\vec x}\psi_0= 4\pi\rho.
\end{cases}
\end{equation}
that clearly satisfies:
\begin{equation}\label{divreducedcurrentPPNaaoiuui111}
\Div_{\vec x}\vec {\widetilde j}\equiv 0.
\end{equation}
and
\begin{equation}\label{reducedcurrentPPNuighjhjaaioiu111}
\vec {\widetilde j}:=(\vec j-\rho\vec
v)-\frac{1}{4\pi}\left(\frac{\partial}{\partial t}
\left(\nabla_{\vec x}\psi_0\right)-curl_{\vec x}\left(\vec v\times
\nabla_{\vec x}\psi_0\right)+\left(\Div_{\vec x}\left\{\nabla_{\vec
x}\psi_0\right\}\right)\vec v\right),
\end{equation}
In particular since by \er{reducedcurrentPPNuighjhjaaioiu111} it can
be easily deduced that the reduced current $\vec {\widetilde j}$ is
a proper vector field, and as before, we deduce that
\er{guigjgjffghguygjyfPPNyuuyuyyuyuuyuyjjijhhih} is invariant under
the change of inertial or non-inertial Cartesian coordinate systems.
Next, defining the following alternative Lagrangian density:
\begin{multline}\label{vhfffngghkjgghPPNintyyyyuyyuyuyuyyy8989988989}
\tilde{L}_1\left(\vec A,\vec C,\vec x,t\right):=
\frac{1}{8\pi}\left(-\nabla_{\vec x}\left(\frac{1}{c}\vec v\cdot\vec
A\right)-\frac{1}{c}\frac{\partial\vec A}{\partial
t}+\frac{1}{c}\vec v\times curl_{\vec x}\vec
A\right)\cdot\left(curl_{\vec x}\vec
C\right)\\-\frac{1}{8\pi}\left(-\nabla_{\vec x}\left(\frac{1}{c}\vec
v\cdot\vec C\right)-\frac{1}{c}\frac{\partial\vec C}{\partial
t}+\frac{1}{c}\vec v\times curl_{\vec x}\vec
C\right)\cdot\left(curl_{\vec x}\vec A\right)
\\-\frac{1}{8\pi}\left|curl_{\vec x}\vec
A\right|^2-\frac{1}{8\pi}\left|curl_{\vec x}\vec
C\right|^2+\frac{1}{c}\vec A\cdot\vec {\widetilde j}.
\end{multline}
and considering the functional
\begin{equation}\label{btfffygtgyggyPPNuyyuyuyuuuiuu}
\tilde{J}=\int_0^T\int_{\mathbb{R}^3}L_1\left(\vec A,\vec C,\vec
x,t\right)d\vec x dt,
\end{equation}
we can easily deduce that $\tilde J$ admits the last two equations
in \er{guigjgjffghguygjyfPPNyuuyuyyuyuuyuyjjijhhih} as the
Euler-Lagrange identities. Here $\psi_0$ is assumed to be fixed and
given by the first equation in
\er{guigjgjffghguygjyfPPNyuuyuyyuyuuyuyjjijhhih}. Moreover, as
before, we deduce that the Lagrangian density $\tilde{L}_1$ in
\er{vhfffngghkjgghPPNintyyyyuyyuyuyuyyy8989988989} is also invariant
under the change of inertial or non-inertial Cartesian coordinate
systems.

Next denoting complex proper vector fields:
\begin{equation}\label{reducedcurrentfhfhjfhjGGaauuuikjjk333}
\vec T:=\vec A-i\vec C\quad\text{ana}\quad\ov{\vec T}:=\vec A+i\vec
C\,,
\end{equation}
we rewrite \er{vhfffngghkjgghPPNintyyyyuyyuyuyuyyy8989988989} as
\begin{multline}\label{vhfffngghkjgghPPNintyyyyuyyuyuyuyyy8989hjjhhj333}
\tilde{L}_1\left(\vec T,\vec x,t\right):=
\frac{i}{16\pi}\left(\nabla_{\vec x}\left(\frac{1}{c}\vec v\cdot\vec
T\right)+\frac{1}{c}\frac{\partial\vec T}{\partial
t}-\frac{1}{c}\vec v\times curl_{\vec x}\vec
T\right)\cdot\left(curl_{\vec x}\ov{\vec
T}\right)\\-\frac{i}{16\pi}\left(\nabla_{\vec
x}\left(\frac{1}{c}\vec v\cdot\ov{\vec
T}\right)+\frac{1}{c}\frac{\partial\ov{\vec T}}{\partial
t}-\frac{1}{c}\vec v\times curl_{\vec x}\ov{\vec
T}\right)\cdot\left(curl_{\vec x}\vec T\right)
\\-\frac{1}{8\pi}\left(curl_{\vec x}\vec
T\right)\cdot\left(curl_{\vec x}\ov{\vec
T}\right)+\frac{1}{2c}\left(\vec T+\ov{\vec T}\right)\cdot\vec
{\widetilde j}.
\end{multline}
and the the last two equations in
\er{guigjgjffghguygjyfPPNyuuyuyyuyuuyuyjjijhhih} together as:
\begin{multline}\label{vhfffngghkjgghPPNintyyyyuyyuyuyuyyy8989hjjhhj888555555555}
\frac{i}{c}\frac{\partial}{\partial t}\left\{curl_{\vec x}\vec
T\right\}-\frac{i}{c}curl_{\vec x}\left\{\vec v\times curl_{\vec
x}\vec T\right\}\\=i\,curl_{\vec x}\left\{\nabla_{\vec
x}\left(\frac{1}{c}\vec v\cdot\vec
T\right)+\frac{1}{c}\frac{\partial\vec T}{\partial
t}-\frac{1}{c}\vec v\times curl_{\vec x}\vec
T\right)=curl\left\{curl_{\vec x}\vec
T\right\}-\frac{4\pi}{c}\vec{\widetilde j}\,.
\end{multline}

\section{Local gravitational time and Maxwell equations in a
non-rotating coordinate system}\label{GIGIGU}
Throughout this
section consider an inertial or more generally a non-rotating
cartesian coordinate system $\bf{(*)}$. Then, as before, in this
system we have
\begin{equation}\label{jhhjgjhuiiu}
\vec v(\vec x,t)=\nabla_{\vec x}Z(\vec x,t),
\end{equation}
where $\vec v$ is the vectorial gravitational potential and $Z$ is a
scalar field. Then define a scalar field $\tau:=\tau(\vec x,t)$ by
the following:
\begin{equation}\label{jhhjgjhuiiuiy}
\tau(\vec x,t)=t+\frac{1}{c^2}Z(\vec x,t).
\end{equation}
We call the quantity $\tau(\vec x,t)$ by the name local
gravitational time. The name "local" and "gravitational" is quite
clear, since $\tau$ depend on the space and time variables and
derived by characteristic function $Z$ of the gravitational field.
The name "time" will be clarified bellow. Note also that, using
\er{noninchgravortbstrjgghguittu1intmmjhhj} in remark \ref{ugyugg},
one can easily deduce that under the change of inertial coordinate
system $\bf{(*)}$ to $\bf{(**)}$ given by the Galilean
Transformation
\begin{equation}\label{noninchgravortbstrjgghguittu1intmmkkkk}
\begin{cases}
\vec x'=\vec x+\vec wt,\\
t'=t,
\end{cases}
\end{equation}
the local gravitational time $\tau$ transforms as:
\begin{multline}\label{noninchgravortbstrjgghguittu1intmmjhhjjj}
\tau'(\vec x',t'):= t'+\frac{1}{c^2}Z'(\vec
x',t')=\left(1+\frac{|\vec w|^2}{2c^2}\right)t+\frac{1}{c^2}Z(\vec
x,t)+ \frac{1}{c^2}\vec w\cdot\vec x\\=\tau(\vec
x,t)+\frac{1}{c^2}\vec w\cdot\vec x+\frac{|\vec w|^2}{2c^2}t \approx
\tau(\vec x,t)+\frac{1}{c^2}\vec w\cdot\vec x,
\end{multline}
where the last equality in
\er{noninchgravortbstrjgghguittu1intmmjhhjjj} is valid if
$\frac{|\vec w|^2}{c^2}\ll 1$. So, under
\er{noninchgravortbstrjgghguittu1intmmkkkk} we have:
\begin{equation}\label{noninchgravortbstrjgghguittu1intmmjhhjhjhj}
\tau'\,=\,\tau+\frac{1}{c^2}\vec w\cdot\vec x+\frac{|\vec
w|^2}{2c^2}t \,\approx\, \tau+\frac{1}{c^2}\vec w\cdot\vec x,
\end{equation}
Next consider the Maxwell equations in the vacuum of the form:
\begin{equation}\label{MaxMedFullGGffgguiuiouioKK}
\begin{cases}
curl_{\vec x} \vec H= \frac{4\pi}{c}\vec j+
\frac{1}{c}\frac{\partial \vec D}{\partial t},\\
div_{\vec x}\vec D= 4\pi\rho,\\
curl_{\vec x} \vec E+\frac{1}{c}\frac{\partial \vec B}{\partial t}=0,\\
div_{\vec x} \vec B=0,\\
\vec E=
\vec D-\frac{1}{c}\,\vec v\times \vec B,\\
\vec H=
\vec B+\frac{1}{c}\,\vec v\times \vec D.
\end{cases}
\end{equation}
where $\vec E$ is the electric field, $\vec B$ is the magnetic
field, $\vec D$ is the electric displacement field and $\vec H$ is
the $\vec H$-magnetic field, $\vec v:=\vec v(\vec x,t)=\nabla_{\vec
x}Z(\vec x,t)$ is the vectorial gravitational potential, $\rho$ is
the charge density and $\vec j$ is the current density. Next
consider a curvilinear change of variables given by:
\begin{equation}\label{giuuihjghgghjgj78zzrrZZffhhhggygghghjhvbKK}
\begin{cases}
t'=\tau(\vec x,t):=t+\frac{Z(\vec x,t)}{c^2}\\
\vec x'=\vec x.
\end{cases}
\end{equation}
Then by the chain rule, for every vector field $\vec F$ we have:
\begin{equation}\label{giuuihjghgghjgj78zzrrZZffhhhggygghghjhvbKKkk}
\begin{cases}
\frac{\partial \vec F}{\partial t}=\frac{\partial \vec F}{\partial
t'}\left(1+\frac{1}{c^2}\frac{\partial Z}{\partial t}\right),\\
d_{\vec x}\vec F=d_{\vec x'}\vec F+\frac{1}{c^2}\frac{\partial \vec
F}{\partial
t'}\otimes\nabla_{\vec x}Z,\\
div_{\vec x}\vec F=div_{\vec x'}\vec F+\frac{1}{c^2}\frac{\partial
\vec
F}{\partial t'}\cdot\nabla_{\vec x}Z,\\
curl_{\vec x}\vec F=curl_{\vec x'}\vec F+\frac{1}{c^2}\nabla_{\vec
x}Z\times\frac{\partial \vec F}{\partial t'}.
\end{cases}
\end{equation}
Thus inserting \er{giuuihjghgghjgj78zzrrZZffhhhggygghghjhvbKKkk}
into \er{MaxMedFullGGffgguiuiouioKK}, since $\vec v=\nabla_{\vec
x}Z$, we deduce
\begin{equation}\label{MaxMedFullGGffgguiuiouiogghghKK}
\begin{cases}
curl_{\vec x'}\vec H+\frac{1}{c^2}\vec v\times\frac{\partial \vec
H}{\partial t'}= \frac{4\pi}{c}\vec j+ \frac{1}{c}\frac{\partial
\vec D}{\partial
t'}\left(1+\frac{1}{c^2}\frac{\partial Z}{\partial t}\right),\\
div_{\vec x'}\vec D+\frac{1}{c^2}\frac{\partial \vec
D}{\partial t'}\cdot\vec v= 4\pi\rho,\\
curl_{\vec x'}\vec E+\frac{1}{c^2}\vec v\times\frac{\partial \vec
E}{\partial t'}+\frac{1}{c}\frac{\partial \vec B}{\partial
t'}\left(1+\frac{1}{c^2}\frac{\partial Z}{\partial t}\right)=0,\\
div_{\vec x'}\vec B+\frac{1}{c^2}\frac{\partial \vec
B}{\partial t'}\cdot\vec v=0,\\
\vec E=\vec D-\frac{1}{c}\,\vec v\times \vec B,\\
\vec H=\vec B+\frac{1}{c}\,\vec v\times \vec D.
\end{cases}
\end{equation}
In particular if $Z$ is independent of $t$ or quasistatic then we
rewrite \er{MaxMedFullGGffgguiuiouiogghghKK} as:
\begin{equation}\label{MaxMedFullGGffgguiuiouiogghghhgghhjhjKK}
\begin{cases}
curl_{\vec x'}\vec H= \frac{4\pi}{c}\vec j+
\frac{1}{c}\frac{\partial }{\partial
t'}\left(\vec D-\frac{1}{c}\,\vec v\times \vec H\right),\\
div_{\vec x'}\vec D+\frac{1}{c}\vec v\cdot curl_{\vec x'}\vec H= 4\pi\left(\rho+\frac{1}{c^2}\,\vec v\cdot\vec j\right),\\
curl_{\vec x'}\vec E+\frac{1}{c}\frac{\partial}{\partial
t'}\left(\vec B+\frac{1}{c}\,\vec v\times \vec E\right)=0,\\
div_{\vec x'}\vec B-\frac{1}{c}\vec v\cdot curl_{\vec x'}\vec E=0,\\
\vec E=\vec D-\frac{1}{c}\,\vec v\times \vec B,\\
\vec H=\vec B+\frac{1}{c}\,\vec v\times \vec D.
\end{cases}
\end{equation}
I.e.
\begin{equation}\label{MaxMedFullGGffgguiuiouiogghghhgghhjhjhjjgKK}
\begin{cases}
curl_{\vec x'}\vec H= \frac{4\pi}{c}\vec j+
\frac{1}{c}\frac{\partial }{\partial
t'}\left(\vec D-\frac{1}{c}\,\vec v\times \vec H\right),\\
div_{\vec x'}\left(\vec D-\frac{1}{c}\,\vec v\times \vec H\right)= 4\pi\left(\rho+\frac{1}{c^2}\,\vec v\cdot\vec j\right),\\
curl_{\vec x'}\vec E+\frac{1}{c}\frac{\partial}{\partial
t'}\left(\vec B+\frac{1}{c}\,\vec v\times \vec E\right)=0,\\
div_{\vec x'}\left(\vec B+\frac{1}{c}\,\vec v\times \vec E\right)=0,\\
\vec E=\vec D-\frac{1}{c}\,\vec v\times \left(\vec H-\frac{1}{c}\,\vec v\times \vec D\right),\\
\vec H=\vec B+\frac{1}{c}\,\vec v\times \left(\vec
E+\frac{1}{c}\,\vec v\times \vec B\right)
\\
\vec E=\vec D-\frac{1}{c}\,\vec v\times \vec B,\\
\vec H=\vec B+\frac{1}{c}\,\vec v\times \vec D.
\end{cases}
\end{equation}
In particular, denoting
\begin{equation}\label{giuuihjghgghjgj78zzrrZZffhhhggygghghjhvbgghhjyuuyKK}
\begin{cases}
\vec E^*:=\vec D-\frac{1}{c}\,\vec v\times \vec H=\vec
E-\frac{1}{c^2}\vec v\times\left(\vec v\times\vec D\right)
\\
\vec H^*:=\vec B+\frac{1}{c}\,\vec v\times \vec E=\vec
H-\frac{1}{c^2}\vec v\times\left(\vec v\times\vec B\right),
\end{cases}
\end{equation}
by \er{MaxMedFullGGffgguiuiouiogghghhgghhjhjhjjgKK} we rewrite the
Maxwell equations in the new curvilinear coordinates in the case of
time independent $\vec v$ as:
\begin{equation}\label{MaxMedFullGGffgguiuiouiogghghhgghhjhjhjjgghhgKK}
\begin{cases}
curl_{\vec x'}\vec H= \frac{4\pi}{c}\vec j+
\frac{1}{c}\frac{\partial \vec E^*}{\partial
t'},\\
div_{\vec x'}\vec E^*= 4\pi\left(\rho+\frac{1}{c^2}\,\vec v\cdot\vec j\right),\\
curl_{\vec x'}\vec E+\frac{1}{c}\frac{\partial\vec H^*}{\partial
t'}=0,\\
div_{\vec x'}\vec H^*=0,\\
\vec E^*=\vec E-\frac{1}{c^2}\vec v\times\left(\vec v\times\vec
D\right)
\\
\vec H^*=\vec H-\frac{1}{c^2}\vec v\times\left(\vec v\times\vec
B\right)\\
\vec E=\vec D-\frac{1}{c}\,\vec v\times \vec B,\\
\vec H=\vec B+\frac{1}{c}\,\vec v\times \vec D.
\end{cases}
\end{equation}
In particular, in the approximation, up to the order
$\left(\frac{|\vec v|}{c}\right)^2\ll 1$ we have $\vec
E^*\approx\vec E$ and $\vec H^*\approx\vec H$ and then the
approximate Maxwell equations have the form:
\begin{equation}\label{MaxMedFullGGffgguiuiouiogghghhgghhjhjhjjgghhgjhgghhhjkiljklKK}
\begin{cases}
curl_{\vec x'}\vec H= \frac{4\pi}{c}\vec j+
\frac{1}{c}\frac{\partial \vec E}{\partial
t'},\\
div_{\vec x'}\vec E= 4\pi\left(\rho+\frac{1}{c^2}\,\vec v\cdot\vec j\right),\\
curl_{\vec x'}\vec E+\frac{1}{c}\frac{\partial\vec H}{\partial
t'}=0,\\
div_{\vec x'}\vec H=0,\\
\vec E=\vec D-\frac{1}{c}\,\vec v\times \vec B,\\
\vec H=\vec B+\frac{1}{c}\,\vec v\times \vec D.
\end{cases}
\end{equation}
The first four equations in
\er{MaxMedFullGGffgguiuiouiogghghhgghhjhjhjjgghhgjhgghhhjkiljklKK}
form a following system of equation:
\begin{equation}\label{MaxMedFullGGffgguiuiouiogghghhgghhjhjhjjgghhgjhgghhhjkiljklKKHH}
\begin{cases}
curl_{\vec x'}\vec H= \frac{4\pi}{c}\vec j^*+
\frac{1}{c}\frac{\partial \vec E}{\partial
t'},\\
div_{\vec x'}\vec E= 4\pi\rho^*,\\
curl_{\vec x'}\vec E+\frac{1}{c}\frac{\partial\vec H}{\partial
t'}=0,\\
div_{\vec x'}\vec H=0,
\end{cases}
\end{equation}
where
\begin{equation}\label{fgjyhyfgfjjjk}
\vec j^*:=\vec
j\quad\text{and}\quad\rho^*:=\left(\rho+\frac{1}{c^2}\,\vec
v\cdot\vec j\right)
\end{equation}
The system
\er{MaxMedFullGGffgguiuiouiogghghhgghhjhjhjjgghhgjhgghhhjkiljklKKHH}
coincides with the classical Maxwell equations of the usual
Electrodynamics and is similar to \er{MaxMedFullGGffgguiuiouioKK}
for the case $\vec v\equiv 0$. Therefore, given known $\vec v$,
$\rho$ and $\vec j$,
\er{MaxMedFullGGffgguiuiouiogghghhgghhjhjhjjgghhgjhgghhhjkiljklKKHH}
could be solved as easy as the usual wave equation, for example by
the method of retarded potentials. Then backward to
\er{giuuihjghgghjgj78zzrrZZffhhhggygghghjhvbKK} change of variables
could be made in order to deduce the electromagnetic fields in
coordinates $(\vec x,t)$. Next note that, since we defined $t'=\tau$
all the above clarifies the name "time" of the quantity $\tau$.
Finally, we would like to note that if we have a motion of some
material body with the place $\vec r(t)$ and the velocity $\vec
u(t):=\frac{d\vec r}{dt}(t)$ and we associate the local
gravitational time $\tau$ with this body then clearly
\begin{equation}\label{fgjyhyfgfjjjkjkkj}
d\tau\,=\, \left(1+\frac{1}{c^2}\vec u(t)\cdot \vec v\left(\vec
r(t),t\right)\right)\,dt\,\approx\, dt,
\end{equation}
where the last equality in \er{fgjyhyfgfjjjkjkkj} is valid if we
have
\begin{equation}\label{fgjyhyfgfjjjkkkk}
\left(\frac{|\vec v|}{c}\right)^2\ll
1\quad\text{and}\quad\left(\frac{|\vec u(t)|}{c}\right)^2\ll 1.
\end{equation}
So we can use the local gravitational time $\tau$ in the approximate
calculations instead of the true time $t$.

\section{Motion of particles in external gravitational-electromagnetic field}
\subsection{Lagrangian of the motion of a finite system of classical
particles in an outer gravitational-electromagnetic
field}\label{hggyugyuy} Given a system of $n$ particles with
inertial masses $m_1,\ldots,m_n$, charges
$\sigma_1,\ldots,\sigma_n$, places $\vec r_1(t),\ldots,\vec r_n(t)$
and velocities $\vec r'_1(t),\ldots, \vec r'_n(t)$ in the outer
gravitational field with vectorial potential $\vec v(\vec x,t)$, the
outer electromagnetical fields with potentials $\vec A(\vec x,t)$
and $\vec \Psi(\vec x,t)$ and additional conservative field with the
classical scalar potential $V(\vec y_1,\ldots,\vec y_n,t)$, consider
a Lagrangian:
\begin{multline}\label{vhfffngghkjgghfjjSYSPN}
L_0\left(\frac{d\vec r_1}{dt},\ldots,\frac{d\vec r_n}{dt},\vec
r_1,\ldots,\vec r_n,t\right):=\\
\sum_{j=1}^{n}\left\{\frac{m_j}{2}\left|\frac{d\vec r_j}{dt}-\vec
v(\vec r_j,t)\right|^2-\sigma_j\left(\Psi(\vec
r_j,t)-\frac{1}{c}\vec A(\vec r_j,t)\cdot\frac{d\vec
r_j}{dt}\right)\right\}+V\left(\vec r_1,\ldots,\vec r_n,t\right).
\end{multline}
This Lagrangian is invariant under the change of inertial and
non-inertial cartesian coordinate systems. We investigate stationary
points of the functional
\begin{equation}\label{btfffygtgyggyijhhkkSYSPN}
J_0=\int_0^T L_0\left(\frac{d\vec r_1}{dt},\ldots,\frac{d\vec
r_n}{dt},,\vec r_1,\ldots,\vec r_n,t\right)dt.
\end{equation}
Then for every $j=1,\ldots,n$ we have
\begin{multline}\label{vhfffngghkjgghggtghjgfhjoyuiyuyhiyyukukySYSPN}
\frac{\delta L_0}{\delta \vec r_j}=-m_j\frac{d}{dt}\left(\frac{d\vec
r_j}{dt}-\vec v(\vec
r_j,t)\right)-\frac{\sigma_j}{c}\frac{d}{dt}\left(\vec A(\vec
r_j,t)\right)-m_j\left\{\nabla_{\vec x}\vec v(\vec
r_j,t)\right\}^T\cdot\left(\frac{d\vec r_j}{dt}-\vec v(\vec
r_j,t)\right)\\-\sigma_j\left(\nabla_{\vec x}\Psi(\vec
r_j,t)-\frac{1}{c}\left\{d_{\vec x}\vec A(\vec
r_j,t)\right\}^T\cdot\frac{d\vec r_j}{dt}\right)+\nabla_{\vec
y_j}V\left(\vec r_1,\ldots,\vec r_n,t\right)=\\-m_j\frac{d^2\vec
r_j}{dt^2}+m_j\left(\frac{\partial}{\partial t}\vec v(\vec
r_j,t)+\nabla_{\vec x}\left(\frac{1}{2}\left|\vec v(\vec
r_j,t)\right|^2\right)-\frac{1}{c}\frac{d\vec r_j}{dt}\times
curl_{\vec x}\vec v(\vec r_j,t)\right)\\+\sigma_j\left(-\nabla_{\vec
x}\Psi(\vec r_j,t)-\frac{1}{c}\frac{\partial}{\partial t}\left(\vec
A(\vec r_j,t)\right)+\frac{1}{c}\frac{d\vec r_j}{dt}\times
curl_{\vec x}\vec A(\vec r_j,t)\right)+\nabla_{\vec y_j}V\left(\vec
r_1,\ldots,\vec r_n,t\right)=0,
\end{multline}
So denoting
\begin{equation}\label{guigjgjffghguygjyfghgghjfgfSYSPN}
\begin{cases}
\vec E=-\nabla_{\vec x}\Psi-\frac{1}{c}\frac{\partial\vec
A}{\partial t}\\
\vec B=curl_{\vec x}\vec A
\end{cases}
\end{equation}
we rewrite \er{vhfffngghkjgghggtghjgfhjoyuiyuyhiyyukukySYSPN} as
\begin{multline}\label{vhfffngghkjgghggtghjgfhjoyuiyuyhiyyukukyihyuSYSPN}
m_j\frac{d^2\vec r_j}{dt^2}=m_j\left(\frac{\partial}{\partial t}\vec
v(\vec r_j,t)+\nabla_{\vec x}\left(\frac{1}{2}\left|\vec v(\vec
r_j,t)\right|^2\right)-\frac{d\vec r_j}{dt}\times curl_{\vec x}\vec
v(\vec r_j,t)\right)+\sigma_j\vec E(\vec
r_j,t)+\frac{\sigma_j}{c}\frac{d\vec r_j}{dt}\times \vec B(\vec
r_j,t)\\+\nabla_{\vec y_j}V\left(\vec r_1,\ldots,\vec
r_n,t\right)=\sigma_j\vec E(\vec
r_j,t)+\frac{\sigma_j}{c}\frac{d\vec r_j}{dt}\times \vec B(\vec
r_j,t)+\nabla_{\vec y_j}V\left(\vec r_1,\ldots,\vec r_n,t\right)+\\
m_j\left(\frac{\partial}{\partial t}\vec v(\vec r_j,t)+d_{\vec
x}\vec v(\vec r_j,t)\cdot\vec v(\vec r_j,t)-\left(\frac{d\vec
r_j}{dt}-\vec v(\vec r_j,t)\right)\times curl_{\vec x}\vec v(\vec
r_j,t)\right).
\end{multline}
So for each particle we get the second law of Newton, consistent
with
\er{noninchgravortbstrjgghguittu2gjgghhjhghjhjgghgghghghtytythvfghfgghjgg},
including the gravitational and the Lorentz force.

Next, assume that our coordinate system is inertial. Then since by
\er{MaxVacFull1ninshtrgravortghhghgjkgghklhjgkghghjjkjhjkkggjkhjkhjjhhfhjhklkhkhjjklzzzyyyhjggjhgghhjhNWNWNWNWNWBWHWPPN222}
we have the following Hamilton-Jacobi type equation
\begin{equation}
\label{MaxVacFull1ninshtrgravortghhghgjkgghklhjgkghghjjkjhjkkggjkhjkhjjhhfhjhklkhkhjjklzzzyyyhjggjhgghhjhNWNWNWNWNWBWHWPPN222kk}
\begin{cases}
\vec v=\nabla_{\vec x}Z,\\
\frac{\partial Z}{\partial t}+\frac{1}{2}\left|\nabla_{\vec
x}Z\right|^2=-\Phi,
\end{cases}
\end{equation}
where $Z:=Z(\vec x,t)$ is some scalar field and $\Phi$ is the
Newtonian gravitational potential, using \er{vhfffngghkjgghfjjSYSPN}
we deduce:
\begin{multline}\label{vhfffngghkjgghfjjSYSPNkk}
L_0\left(\frac{d\vec r_1}{dt},\ldots,\frac{d\vec r_n}{dt},\vec
r_1,\ldots,\vec r_n,t\right)=\\
\sum_{j=1}^{n}\left\{\frac{m_j}{2}\left|\frac{d\vec
r_j}{dt}-\nabla_{\vec x}Z(\vec
r_j,t)\right|^2-\sigma_j\left(\Psi(\vec r_j,t)-\frac{1}{c}\vec
A(\vec r_j,t)\cdot\frac{d\vec r_j}{dt}\right)\right\}+V\left(\vec
r_1,\ldots,\vec
r_n,t\right)\\=\sum_{j=1}^{n}m_j\left(\frac{1}{2}\left|\nabla_{\vec
x}Z(\vec r_j,t)\right|^2-\nabla_{\vec x}Z(\vec
r_j,t)\cdot\frac{d\vec r_j}{dt}\right)\\+
\sum_{j=1}^{n}\left\{\frac{m_j}{2}\left|\frac{d\vec
r_j}{dt}\right|^2-\sigma_j\left(\Psi(\vec r_j,t)-\frac{1}{c}\vec
A(\vec r_j,t)\cdot\frac{d\vec r_j}{dt}\right)\right\}+V\left(\vec
r_1,\ldots,\vec r_n,t\right)\\
=\sum_{j=1}^{n}m_j\left(\frac{\partial Z}{\partial t}(\vec
r_j,t)+\frac{1}{2}\left|\nabla_{\vec x}Z(\vec
r_j,t)\right|^2\right)-\sum_{j=1}^{n}m_j\frac{d}{dt}\left\{Z\left(\vec
r_j(t),t\right)\right\}\\+
\sum_{j=1}^{n}\left\{\frac{m_j}{2}\left|\frac{d\vec
r_j}{dt}\right|^2-\sigma_j\left(\Psi(\vec r_j,t)-\frac{1}{c}\vec
A(\vec r_j,t)\cdot\frac{d\vec r_j}{dt}\right)\right\}+V\left(\vec
r_1,\ldots,\vec r_n,t\right)=\\
\sum_{j=1}^{n}\left\{\frac{m_j}{2}\left|\frac{d\vec
r_j}{dt}\right|^2-m_j\Phi(\vec r_j,t)-\sigma_j\left(\Psi(\vec
r_j,t)-\frac{1}{c}\vec A(\vec r_j,t)\cdot\frac{d\vec
r_j}{dt}\right)\right\}+V\left(\vec r_1,\ldots,\vec r_n,t\right)\\
-\frac{d}{dt}\left\{\sum_{j=1}^{n}m_jZ\left(\vec
r_j(t),t\right)\right\}.
\end{multline}
So we rewrite \er{vhfffngghkjgghfjjSYSPN} as:
\begin{multline}\label{vhfffngghkjgghfjjSYSPNkkjjkk}
L_0\left(\frac{d\vec r_1}{dt},\ldots,\frac{d\vec r_n}{dt},\vec
r_1,\ldots,\vec r_n,t\right)=L'_0\left(\frac{d\vec
r_1}{dt},\ldots,\frac{d\vec r_n}{dt},\vec r_1,\ldots,\vec
r_n,t\right)-\frac{d}{dt}\left\{\sum_{j=1}^{n}m_jZ\left(\vec
r_j(t),t\right)\right\}.
\end{multline}
where
\begin{multline}\label{vhfffngghkjgghfjjSYSPNkkjjkkhhkk}
L'_0\left(\frac{d\vec r_1}{dt},\ldots,\frac{d\vec r_n}{dt},\vec
r_1,\ldots,\vec r_n,t\right):=\\
\sum_{j=1}^{n}\left\{\frac{m_j}{2}\left|\frac{d\vec
r_j}{dt}\right|^2-m_j\Phi(\vec r_j,t)-\sigma_j\left(\Psi(\vec
r_j,t)-\frac{1}{c}\vec A(\vec r_j,t)\cdot\frac{d\vec
r_j}{dt}\right)\right\}+V\left(\vec r_1,\ldots,\vec r_n,t\right).
\end{multline}
Note that in the given inertial coordinate system $L'_0$ coincides
with the classical Lagrangian of motion in the gravitational and
electromagnetic fields . Moreover, we rewrite
\er{btfffygtgyggyijhhkkSYSPN} as:
\begin{equation}\label{btfffygtgyggyijhhkkSYSPNkk}
J_0=\int_0^T L'_0\left(\frac{d\vec r_1}{dt},\ldots,\frac{d\vec
r_n}{dt},\vec r_1,\ldots,\vec
r_n,t\right)dt+\sum_{j=1}^{n}m_j\left(Z\left(\vec
r_j(0),0\right)-Z\left(\vec r_j(T),T\right)\right).
\end{equation}
Thus the stationary points of the functional $J_0$ will satisfy the
same Euler-Lagrange equations as the stationary points of the
functional
\begin{equation}\label{btfffygtgyggyijhhkkSYSPNkkllkk}
J'_0=\int_0^T L'_0\left(\frac{d\vec r_1}{dt},\ldots,\frac{d\vec
r_n}{dt},\vec r_1,\ldots,\vec r_n,t\right)dt,
\end{equation}
provided that the beginning and the ending points of trajectories
$\vec r_j(t)$ are fixed.

\subsection{Hamiltonian of the motion of a finite system of classical
particles in an outer gravitational-electromagnetic
field}\label{hggyugyuyhkhj} For every $j=1,\ldots,n$ define the
generalized momentum of the particle $m_j$ by
\begin{equation}\label{guytyurtydftyiujhSYSPN}
\vec P_j:=\nabla_{\vec r'_j}L_0\left(\vec r'_1,\ldots,\vec r'_n,\vec
r_1,\ldots,\vec r_n,t\right)=m_j \frac{d\vec r_j}{dt}-m_j\vec v(\vec
r_j,t)+\frac{\sigma_j}{c}\vec A(\vec r_j,t),
\end{equation}
where $L_0$ is given by \er{vhfffngghkjgghfjjSYSPN}. Then
\begin{equation}\label{guytyurtydftyiujhioyhuyhSYSPN}
\frac{d\vec r_j}{dt}=\frac{1}{m_j}\vec P_j+\vec v(\vec
r_j,t)-\frac{\sigma_j}{m_jc}\vec A(\vec r_j,t).
\end{equation}
Thus if we consider a Hamiltonian
\begin{equation}\label{vhfffngghkjgghfjjhyjjfgSYSPN}
H_0\left(\vec P_1,\ldots,\vec P_n,\vec r_1,\ldots,\vec
r_n,t\right):=\sum_{j=1}^{n}\vec P_j\cdot\frac{d\vec
r_j}{dt}-L_0\left(\frac{d\vec r_1}{dt},\ldots,\frac{d\vec
r_n}{dt},\vec r_1,\ldots,\vec r_n,t\right)
\end{equation}
then by \er{vhfffngghkjgghfjjSYSPN},
\er{vhfffngghkjgghfjjhyjjfgSYSPN} and
\er{guytyurtydftyiujhioyhuyhSYSPN} we have:
\begin{multline}\label{vhfffngghkjgghfjjghghghSYSPN}
H_0\left(\vec P_1,\ldots,\vec P_n,\vec r_1,\ldots,\vec
r_n,t\right)=-V\left(\vec r_1,\ldots,\vec
r_n,t\right)+\sum_{j=1}^{n}\vec P_j\cdot\frac{d\vec
r_j}{dt}\\
-\sum_{j=1}^{n}\left(\frac{m_j}{2}\left|\frac{d\vec r_j}{dt}-\vec
v(\vec r_j,t)\right|^2-\sigma_j\left(\Psi(\vec
r_j,t)-\frac{1}{c}\vec A(\vec
r_j,t)\cdot\frac{d\vec r_j}{dt}\right)\right)=\\
\sum_{j=1}^{n}\vec P_j\cdot\left(\frac{1}{m_j}\vec P_j+\vec v(\vec
r_j,t)-\frac{\sigma_j}{m_jc}\vec A(\vec
r_j,t)\right)-\sum_{j=1}^{n}\frac{m_j}{2}\left|\frac{1}{m_j}\vec
P_j-\frac{\sigma_j}{m_jc}\vec A(\vec
r_j,t)\right|^2\\+\sum_{j=1}^{n}\sigma_j\left(\Psi(\vec
r_j,t)-\frac{1}{c}\vec A(\vec r_j,t)\cdot\left(\frac{1}{m_j}\vec
P_j+\vec v(\vec r_j,t)-\frac{\sigma_j}{m_jc}\vec A(\vec
r_j,t)\right)\right)-V\left(\vec r_1,\ldots,\vec r_n,t\right) =\\
\sum_{j=1}^{n}\vec P_j\cdot\vec v(\vec r_j,t)+
\sum_{j=1}^{n}\frac{1}{2m_j}\left|\vec P_j-\frac{\sigma_j}{c}\vec
A(\vec r_j,t)\right|^2+\sum_{j=1}^{n}\sigma_j\left(\Psi(\vec
r_j,t)-\frac{1}{c}\vec A(\vec r_j,t)\cdot\vec v(\vec
r_j,t)\right)-V\left(\vec r_1,\ldots,\vec r_n,t\right).
\end{multline}

\subsection{Classical Liouville's equation}\label{gghfhfhfhfh}
Assume that the number of particles $n$ in the system, ruled by the
Hamiltonian \er{vhfffngghkjgghfjjghghghSYSPN}, is large and we
describe this system statistically. Then let $w\left(\vec
P_1,\ldots,\vec P_n,\vec r_1,\ldots,\vec r_n,t\right)\to[0,+\infty)$
be the probability density of the system which satisfies the well
known classical Liouville's equation of the form:
\begin{multline}\label{vhfffngghkjgghfjjghghghhjghjgghkghggSYSPNkljljjmjkj}
\frac{\partial w}{\partial t}+\sum_{j=1}^{n}\left(div_{\vec
r_j}\left\{w\,\nabla_{\vec
P_j}H_0\right\}-div_{\vec P_j}\left\{w\,\nabla_{\vec r_j}H_0\right\}\right)=\\
\frac{\partial w}{\partial t}+\sum_{j=1}^{n}\left(\nabla_{\vec
P_j}H_0\cdot\nabla_{\vec r_j}w-\nabla_{\vec r_j}H_0\cdot\nabla_{\vec
P_j}w\right)=0.
\end{multline}
Then since by \er{vhfffngghkjgghfjjghghghSYSPN} we have
\begin{multline}\label{vhfffngghkjgghfjjghghghSYSPNjljl}
H_0\left(\vec P_1,\ldots,\vec P_n,\vec r_1,\ldots,\vec
r_n,t\right)=\\
\sum_{j=1}^{n}\vec P_j\cdot\vec v(\vec r_j,t)+
\sum_{j=1}^{n}\frac{1}{2m_j}\left|\vec P_j-\frac{\sigma_j}{c}\vec
A(\vec r_j,t)\right|^2+\sum_{j=1}^{n}\sigma_j\left(\Psi(\vec
r_j,t)-\frac{1}{c}\vec A(\vec r_j,t)\cdot\vec v(\vec
r_j,t)\right)-V\left(\vec r_1,\ldots,\vec r_n,t\right),
\end{multline}
and in particular,
\begin{multline}\label{vhfffngghkjgghfjjghghghhjghjgghkghggSYSPNkljljjmjkjjjkkjlkjkljljjkjkkjjkl}
\nabla_{\vec P_j}H_0=\frac{1}{m_j}\left(\vec
P_j-\frac{\sigma_j}{c}\vec A(\vec r_j,t)\right)+\vec v( \vec
r_j,t)\quad\text{and}\\ \nabla_{\vec r_j}H_0= \left\{d_{\vec x}\vec
v(\vec r_j,t)\right\}^T\cdot\vec P_j-\frac{\sigma_j}{c
m_j}\left\{d_{\vec x}\vec A(\vec r_j,t)\right\}^T\cdot\left(\vec
P_j-\frac{\sigma_j}{c}\vec A(\vec
r_j,t)\right)\\+\sigma_j\nabla_{\vec x}\left(\Psi(\vec
r_j,t)-\frac{1}{c}\vec A(\vec r_j,t)\cdot\vec v(\vec
r_j,t)\right)-\nabla_{\vec y_j}V\left(\vec r_1,\ldots,\vec
r_n,t\right)\quad\quad\forall j=1,\ldots n,
\end{multline}
inserting
\er{vhfffngghkjgghfjjghghghhjghjgghkghggSYSPNkljljjmjkjjjkkjlkjkljljjkjkkjjkl}
into \er{vhfffngghkjgghfjjghghghhjghjgghkghggSYSPNkljljjmjkj} gives
\begin{multline}\label{vhfffngghkjgghfjjghghghhjghjgghkghggSYSPNkljljjmjkjhhhjhigjgj}
0=\frac{\partial w}{\partial t}+\sum_{j=1}^{n}\nabla_{\vec
P_j}H_0\cdot\nabla_{\vec r_j}w-\sum_{j=1}^{n}\nabla_{\vec
r_j}H_0\cdot\nabla_{\vec P_j}w=\\
\frac{\partial w}{\partial t} +\sum_{j=1}^{n}\vec v( \vec
r_j,t)\cdot\nabla_{\vec r_j}w +\sum_{j=1}^{n}\frac{1}{m_j}\left(\vec
P_j-\frac{\sigma_j}{c}\vec A(\vec r_j,t)\right)\cdot\nabla_{\vec
r_j}w -\sum_{j=1}^{n}\left(\left\{d_{\vec x}\vec v(\vec
r_j,t)\right\}^T\cdot\vec P_j\right)\cdot\nabla_{\vec
P_j}w\\
+ \sum_{j=1}^{n}\left(\frac{\sigma_j}{c m_j}\left\{d_{\vec x}\vec
A(\vec r_j,t)\right\}^T\cdot\left(\vec P_j-\frac{\sigma_j}{c}\vec
A(\vec r_j,t)\right)\right)\cdot\nabla_{\vec P_j}w
\\
-\sum_{j=1}^{n}\left(\sigma_j\nabla_{\vec x}\left(\Psi(\vec
r_j,t)-\frac{1}{c}\vec A(\vec r_j,t)\cdot\vec v(\vec
r_j,t)\right)-\nabla_{\vec y_j}V\left(\vec r_1,\ldots,\vec
r_n,t\right)\right)\cdot\nabla_{\vec P_j}w
=\\
\frac{\partial w}{\partial t} +\sum_{j=1}^{n}\vec v( \vec
r_j,t)\cdot\nabla_{\vec r_j}w +\sum_{j=1}^{n}\frac{1}{m_j}\left(\vec
P_j-\frac{\sigma_j}{c}\vec A(\vec r_j,t)\right)\cdot\nabla_{\vec
r_j}w \\+\sum_{j=1}^{n}\frac{1}{2}\left(\left(d_{\vec x}\vec v(\vec
r_j,t)-\left\{d_{\vec x}\vec v(\vec r_j,t)\right\}^T\right)\cdot\vec
P_j\right)\cdot\nabla_{\vec
P_j}w-\sum_{j=1}^{n}\frac{1}{2}\left(\left(d_{\vec x}\vec v(\vec
r_j,t)+\left\{d_{\vec x}\vec v(\vec r_j,t)\right\}^T\right)\cdot\vec
P_j\right)\cdot\nabla_{\vec
P_j}w\\
+ \sum_{j=1}^{n}\left(\frac{\sigma_j}{c m_j}\left\{d_{\vec x}\vec
A(\vec r_j,t)\right\}^T\cdot\left(\vec P_j-\frac{\sigma_j}{c}\vec
A(\vec r_j,t)\right)\right)\cdot\nabla_{\vec P_j}w
\\
-\sum_{j=1}^{n}\left(\sigma_j\nabla_{\vec x}\left(\Psi(\vec
r_j,t)-\frac{1}{c}\vec A(\vec r_j,t)\cdot\vec v(\vec
r_j,t)\right)-\nabla_{\vec y_j}V\left(\vec r_1,\ldots,\vec
r_n,t\right)\right)\cdot\nabla_{\vec P_j}w.
\end{multline}
Thus, by
\er{vhfffngghkjgghfjjghghghhjghjgghkghggSYSPNkljljjmjkjhhhjhigjgj},
using \er{apfrm9}, we rewrite the Liouville's equation as:
\begin{multline}\label{vhfffngghkjgghfjjghghghhjghjgghkghggSYSPNkljljjmjkjhhhjhigjgjjhjgjg}
\frac{\partial w}{\partial t} +\sum_{j=1}^{n}\vec v( \vec
r_j,t)\cdot\nabla_{\vec r_j}w
+\sum_{j=1}^{n}\frac{1}{2}\left(\left(curl_{\vec x}\vec v(\vec
r_j,t)\right)\times\vec P_j\right)\cdot\nabla_{\vec P_j}w
+\sum_{j=1}^{n}\frac{1}{m_j}\left(\vec P_j-\frac{\sigma_j}{c}\vec
A(\vec r_j,t)\right)\cdot\nabla_{\vec r_j}w
\\-\sum_{j=1}^{n}\frac{1}{2}\left(\left(d_{\vec x}\vec v(\vec
r_j,t)+\left\{d_{\vec x}\vec v(\vec r_j,t)\right\}^T\right)\cdot\vec
P_j\right)\cdot\nabla_{\vec P_j}w +
\sum_{j=1}^{n}\left(\frac{\sigma_j}{c m_j}\left\{d_{\vec x}\vec
A(\vec r_j,t)\right\}^T\cdot\left(\vec P_j-\frac{\sigma_j}{c}\vec
A(\vec r_j,t)\right)\right)\cdot\nabla_{\vec P_j}w
\\
-\sum_{j=1}^{n}\left(\sigma_j\nabla_{\vec x}\left(\Psi(\vec
r_j,t)-\frac{1}{c}\vec A(\vec r_j,t)\cdot\vec v(\vec
r_j,t)\right)-\nabla_{\vec y_j}V\left(\vec r_1,\ldots,\vec
r_n,t\right)\right)\cdot\nabla_{\vec P_j}w=0.
\end{multline}
Next if the change of some non-inertial cartesian coordinate system
$(*)$ to another cartesian coordinate system $(**)$ is of the form
\er{noninchgravortbstrjgghguittu2}:
\begin{equation}\label{noninchgravortbstrjgghguittu2jjjk}
\begin{cases}
\vec x'=A(t)\cdot\vec x+\vec z(t),\\
t'=t,
\end{cases}
\end{equation}
where $A(t)\in SO(3)$ is a rotation, then consistently with
\er{noninchgravortbstrjgghguittu2jjjk} and
\er{guytyurtydftyiujhSYSPN} we have the following change of
variables $(\vec P_1,\ldots,\vec P_n,\vec x_1,\ldots,\vec
x_n,t)\to(\vec P'_1,\ldots,\vec P'_n,\vec x'_1,\ldots,\vec
x'_n,t')$:
\begin{equation}\label{noninchgravortbstrjgghguittu2intrrrZZHHjkjhjh}
\begin{cases}
t'=t,\\
\vec x'_k=A(t)\cdot\vec x_k+\vec z(t)\quad\quad\forall
k=1,\ldots,n,\\
\vec P'_k=A(t)\cdot\vec P_k\quad\quad\forall k=1,\ldots,n.
\end{cases}
\end{equation}
Thus, since consistently with \er{noninchgravortbstrjgghguittu2jjjk}
we have
\begin{equation}\label{yuythfgfyftydtydtydtyddyyyhhddhhhredPPN111hgghjgintintrrZZHHlll}
\begin{cases}
V'\left(\vec x'_1,\ldots,\vec x'_n,t'\right)\,=\,V\left(\vec
x_1,\ldots,\vec x_n,t\right),\\
\sigma'_j\,=\,\sigma_j,\\
m'_j\,=\,m_j,\\
\vec v'(\vec x',t)\,=\,A(t)\cdot \vec v(\vec x,t)+\frac{dA}{dt}(t)\cdot\vec x+\frac{d\vec z}{dt}(t),\\
\vec A'(\vec x',t)\,=\,A(t)\cdot \vec A(\vec x,t),\\
\Psi'(\vec x',t)-\vec v'(\vec x',t)\cdot\vec A'(\vec
x',t)\,=\,\Psi(\vec x,t)-\vec v(\vec x,t)\cdot\vec A(\vec x,t),
\end{cases}
\end{equation}
by \er{noninchgravortbstrjgghguittu2intrrrZZHHjkjhjh} and
\er{yuythfgfyftydtydtydtyddyyyhhddhhhredPPN111hgghjgintintrrZZHHlll}
we deduce that the Liouville equation
\er{vhfffngghkjgghfjjghghghhjghjgghkghggSYSPNkljljjmjkjhhhjhigjgjjhjgjg}
is invariant under the the change of non-inertial cartesian
coordinate system of the form \er{noninchgravortbstrjgghguittu2}.

\subsubsection{Thermodynamical equilibrium; canonical and micro-canonical
ensembles}\label{hjgutgyutyyt}
Next let $w\left(\vec P_1,\ldots,\vec
P_n,\vec r_1,\ldots,\vec r_n,t\right)\to[0,+\infty)$ be the
probability density of the system, ruled by the Hamiltonian $H_0$
from \er{vhfffngghkjgghfjjghghghSYSPNjljl}, in the case of
\underline{thermodynamical equilibrium}. Then in the case that
$\frac{\partial H_0}{\partial t}\equiv 0$ and the given equilibrium
system rests macroscopically, i.e. it has the macroscopical velocity
field zero: $\vec u(\vec x,t)\equiv 0$ it is well known that
\begin{equation}\label{vhfffngghkjgghfjjghghghSYSPNjljllkljjhkhkhklpl;lpoopkklkljl}
w\left(\vec P_1,\ldots,\vec P_n,\vec r_1,\ldots,\vec
r_n,t\right):=\frac{1}{K(f,H_{0})}f\left(H_{0}\left(\vec
P_1,\ldots,\vec P_n,\vec r_1,\ldots,\vec r_n,t\right)\right),
\end{equation}
with
\begin{equation}\label{vhfffngghkjgghfjjghghghSYSPNjljllkljjhkhkhklpl;liuiouiojjk}
K(f,H_{0}):=\int f\left(H_{0}\left(\vec P_1,\ldots,\vec P_n,\vec
r_1,\ldots,\vec r_n,t\right)\right)d\vec P_1\ldots d\vec P_nd\vec
r_1\ldots d\vec r_n,
\end{equation}
and,
\begin{itemize}
\item in the case of a system in thermostat, having the Kelvin temperature $T$, we have the canonical ensemble: $f(s)=e^{-\frac{s}{kT}}\;\,\forall s$, where $k$ is the Boltzmann
constant,
\item in the case of a thermally isolated system with the average internal energy $E$ we have the micro-canonical ensemble: $f(s)=\delta(s-E)\;\,\forall
s$.
\end{itemize}
We would like to find alternative forms of the above laws of
thermodynamical equilibrium, which are invariant under the change of
inertial or non-inertial cartesian coordinate system. Then it is
clear that if the given system rests in the old coordinate system,
then it obviously has non-trivial macroscopical velocity field $\vec
u(\vec x,t)$ in the new one. Moreover, $\vec u(\vec x,t)$ can depend
on $\vec x$ and $t$, as it indeed happens in the case of a rotation
of the new coordinate system with respect to the old one. On the
other hand the concept of thermodynamical equilibrium is clearly
independent of the coordinate system.

In order to find the forms of the laws of thermodynamical
equilibrium, which are indeed invariant under the change of inertial
or non-inertial cartesian coordinate system we follow the steps as
below. By \er{vhfffngghkjgghfjjghghghSYSPNjljl} the Hamiltonian
$H_0$ of our system has the following form:
\begin{multline}\label{vhfffngghkjgghfjjghghghSYSPNjljll;kl;}
H_0\left(\vec P_1,\ldots,\vec P_n,\vec r_1,\ldots,\vec
r_n,t\right)=\\
\sum_{j=1}^{n}\vec P_j\cdot\vec v(\vec r_j,t)+
\sum_{j=1}^{n}\frac{1}{2m_j}\left|\vec P_j-\frac{\sigma_j}{c}\vec
A(\vec r_j,t)\right|^2+\sum_{j=1}^{n}\sigma_j\left(\Psi(\vec
r_j,t)-\frac{1}{c}\vec A(\vec r_j,t)\cdot\vec v(\vec
r_j,t)\right)-V\left(\vec r_1,\ldots,\vec r_n,t\right),
\end{multline}
Next define an invariant quantity $H_{\vec u}$ as:
\begin{multline}\label{vhfffngghkjgghfjjghghghSYSPNjljll'l}
H_{\vec u}\left(\vec P_1,\ldots,\vec P_n,\vec r_1,\ldots,\vec
r_n,t\right):=\\H_0\left(\vec P_1,\ldots,\vec P_n,\vec
r_1,\ldots,\vec r_n,t\right)-\sum_{j=1}^{n}\vec P_j\cdot\vec u(\vec
r_j,t)=
\sum_{j=1}^{n}\vec P_j\cdot\left(\vec v(\vec r_j,t)-\vec u(\vec
r_j,t)\right)\\+ \sum_{j=1}^{n}\frac{1}{2m_j}\left|\vec
P_j-\frac{\sigma_j}{c}\vec A(\vec
r_j,t)\right|^2+\sum_{j=1}^{n}\sigma_j\left(\Psi(\vec
r_j,t)-\frac{1}{c}\vec A(\vec r_j,t)\cdot\vec v(\vec
r_j,t)\right)-V\left(\vec r_1,\ldots,\vec r_n,t\right),
\end{multline}
where $\vec u(\vec x,t)$ is the macroscopical velocity field of the
given system of particles. Then, as before, it can be easily deduced
that under the change of some non-inertial cartesian coordinate
system $(*)$ to another cartesian coordinate system $(**)$ of the
form \er{noninchgravortbstrjgghguittu2jjjk} the quantity $H_{\vec
u}$ transforms as:
\begin{equation}\label{vhfffngghkjgghfjjghghghSYSPNjljllkljjhkhkhklpl;lpoopkklkljljkhhj}
H'_{\vec u'}=H_{\vec u},
\end{equation}
provided that $\vec u$ is the speed-like vector field i.e. under the
above change we have:
\begin{equation}\label{yuythfgfyftydtydtydtyddyyyhhddhhhredPPN111hgghjgintintrrZZHHlllkljjjk}
\vec u'(\vec x',t)\,=\,A(t)\cdot \vec u(\vec
x,t)+\frac{dA}{dt}(t)\cdot\vec x+\frac{d\vec z}{dt}(t),
\end{equation}
and we have \er{noninchgravortbstrjgghguittu2intrrrZZHHjkjhjh} and
\er{yuythfgfyftydtydtydtyddyyyhhddhhhredPPN111hgghjgintintrrZZHHlll}.
Next consider the canonical and micro-canonical ensembles as:
\begin{equation}\label{vhfffngghkjgghfjjghghghSYSPNjljllkljjhkhkhklpl;lpoop}
w\left(\vec P_1,\ldots,\vec P_n,\vec r_1,\ldots,\vec
r_n,t\right):=\frac{1}{K(f,H_{\vec u})}f\left(H_{\vec u}\left(\vec
P_1,\ldots,\vec P_n,\vec r_1,\ldots,\vec r_n,t\right)\right),
\end{equation}
where as before,
\begin{equation}\label{vhfffngghkjgghfjjghghghSYSPNjljllkljjhkhkhklpl;liuiouio}
K(f,H_{\vec u}):=\int f\left(H_{\vec u}\left(\vec P_1,\ldots,\vec
P_n,\vec r_1,\ldots,\vec r_n,t\right)\right)d\vec P_1\ldots d\vec
P_nd\vec r_1\ldots d\vec r_n,
\end{equation}
and,
\begin{itemize}
\item in the case of a system in thermostat, having the Kelvin temperature $T$, we have: $f(s)=e^{-\frac{s}{kT}}\;\,\forall
s$,
\item in the case of a thermally isolated system with the average internal energy $E$ we have: $f(s)=\delta(s-E)\;\,\forall
s$.
\end{itemize}
Then, by
\er{vhfffngghkjgghfjjghghghSYSPNjljllkljjhkhkhklpl;lpoopkklkljljkhhj}
we deduce that under the change of some non-inertial cartesian
coordinate system $(*)$ to another cartesian coordinate system
$(**)$ of the form \er{noninchgravortbstrjgghguittu2jjjk} the
quantity $w$ transforms as:
\begin{equation}\label{vhfffngghkjgghfjjghghghSYSPNjljllkljjhkhkhklpl;lpoopkklkljljkhhjkhkhk}
w'=w.
\end{equation}
Moreover, in the case $\vec u\equiv 0$
\er{vhfffngghkjgghfjjghghghSYSPNjljllkljjhkhkhklpl;lpoop} coincides
with
\er{vhfffngghkjgghfjjghghghSYSPNjljllkljjhkhkhklpl;lpoopkklkljl}.
However, we still need to derive the restrictions on the field $\vec
u$ and the Hamiltonian $H_0$, providing that our system can indeed
be found in the state of thermodynamical equilibrium. We remind that
in the case $\vec u\equiv 0$ the appropriate restriction is
$\frac{\partial H_0}{\partial t}\equiv 0$. In order to get these
restrictions in the general case, we need to insert $w$ in
\er{vhfffngghkjgghfjjghghghSYSPNjljllkljjhkhkhklpl;lpoop} into the
Liouville's equation in
\er{vhfffngghkjgghfjjghghghhjghjgghkghggSYSPNkljljjmjkj} having the
form:
\begin{equation}\label{vhfffngghkjgghfjjghghghhjghjgghkghggSYSPNkljljjmjkjkjjk}
\frac{\partial w}{\partial t}+\sum_{j=1}^{n}\left(\nabla_{\vec
P_j}H_0\cdot\nabla_{\vec r_j}w-\nabla_{\vec r_j}H_0\cdot\nabla_{\vec
P_j}w\right)=0.
\end{equation}
By \er{vhfffngghkjgghfjjghghghSYSPNjljll'l} we have:
\begin{multline}\label{vhfffngghkjgghfjjghghghhjghjgghkghggSYSPNkljljjmjkjjjkkjlkjkljljjkjkkjjklkjjkjh}
\nabla_{\vec P_j}H_{\vec u}=\nabla_{\vec P_j}H_0-\vec u( \vec
r_j,t)\quad\text{and}\quad \nabla_{\vec r_j}H_{\vec u}=\nabla_{\vec
r_j}H_0-\left\{d_{\vec x}\vec u(\vec r_j,t)\right\}^T\cdot\vec
P_j\quad\quad\forall j=1,\ldots n.
\end{multline}
Next inserting
\er{vhfffngghkjgghfjjghghghSYSPNjljllkljjhkhkhklpl;lpoop} into
\er{vhfffngghkjgghfjjghghghhjghjgghkghggSYSPNkljljjmjkjkjjk} we
deuce:
\begin{equation}\label{vhfffngghkjgghfjjghghghhjghjgghkghggSYSPNkljljjkklmjkjkjjk}
\frac{\partial H_{\vec u}}{\partial
t}+\sum_{j=1}^{n}\left(\nabla_{\vec P_j}H_0\cdot\nabla_{\vec
r_j}H_{\vec u}-\nabla_{\vec r_j}H_0\cdot\nabla_{\vec P_j}H_{\vec
u}\right)=0.
\end{equation}
Thus inserting
\er{vhfffngghkjgghfjjghghghhjghjgghkghggSYSPNkljljjmjkjjjkkjlkjkljljjkjkkjjklkjjkjh}
into \er{vhfffngghkjgghfjjghghghhjghjgghkghggSYSPNkljljjkklmjkjkjjk}
we obtain,
\begin{multline}\label{vhfffngghkjgghfjjghghghhjghjgghkghggSYSPNkljljjkklmjkjkjjk,mmm,kl;k;lklk}
0=\frac{\partial H_{\vec u}}{\partial
t}+\sum_{j=1}^{n}\left(\left(\nabla_{\vec P_j}H_{\vec u}+\vec u(
\vec r_j,t)\right)\cdot\nabla_{\vec r_j}H_{\vec
u}-\left(\nabla_{\vec r_j}H_{\vec u}+\left\{d_{\vec x}\vec u(\vec
r_j,t)\right\}^T\cdot\vec P_j\right)\cdot\nabla_{\vec P_j}H_{\vec
u}\right)=\\
\frac{\partial H_{\vec u}}{\partial t}+\sum_{j=1}^{n}\vec u( \vec
r_j,t)\cdot\nabla_{\vec r_j}H_{\vec
u}-\sum_{j=1}^{n}\left(\left\{d_{\vec x}\vec u(\vec
r_j,t)\right\}^T\cdot\vec P_j\right)\cdot\nabla_{\vec P_j}H_{\vec u}
=\\
\frac{\partial H_{\vec u}}{\partial t}+\sum_{j=1}^{n}\vec u( \vec
r_j,t)\cdot\nabla_{\vec r_j}H_{\vec
u}+\sum_{j=1}^{n}\frac{1}{2}\left(\left(d_{\vec x}\vec u(\vec
r_j,t)-\left\{d_{\vec x}\vec u(\vec r_j,t)\right\}^T\right)\cdot\vec
P_j\right)\cdot\nabla_{\vec P_j}H_{\vec
u}\\-\sum_{j=1}^{n}\frac{1}{2}\left(\left(d_{\vec x}\vec u(\vec
r_j,t)+\left\{d_{\vec x}\vec u(\vec r_j,t)\right\}^T\right)\cdot\vec
P_j\right)\cdot\nabla_{\vec P_j}H_{\vec u}.
\end{multline}
Thus, by \er{apfrm9}, we rewrite
\er{vhfffngghkjgghfjjghghghhjghjgghkghggSYSPNkljljjkklmjkjkjjk,mmm,kl;k;lklk}
as:
\begin{multline}\label{vhfffngghkjgghfjjghghghhjghjgghkghggSYSPNkljljjkklmjkjkjjk,mmm,kl;k;lklkaa}
\frac{\partial H_{\vec u}}{\partial t}+\sum_{j=1}^{n}\vec u( \vec
r_j,t)\cdot\nabla_{\vec r_j}H_{\vec
u}+\sum_{j=1}^{n}\frac{1}{2}\left(\left(curl_{\vec x}\vec u(\vec
r_j,t)\right)\times\vec P_j\right)\cdot\nabla_{\vec P_j}H_{\vec
u}\\-\sum_{j=1}^{n}\frac{1}{2}\left(\left(d_{\vec x}\vec u(\vec
r_j,t)+\left\{d_{\vec x}\vec u(\vec r_j,t)\right\}^T\right)\cdot\vec
P_j\right)\cdot\nabla_{\vec P_j}H_{\vec u}=0.
\end{multline}
Equality
\er{vhfffngghkjgghfjjghghghhjghjgghkghggSYSPNkljljjkklmjkjkjjk,mmm,kl;k;lklkaa}
is the required restriction on the field $\vec u$ and the
Hamiltonian $H_{0}$, providing that our system can indeed be found
in state of thermodynamical equilibrium. In particular, if
$\frac{\partial H_0}{\partial t}=0$ and $\vec u=0$ then
\er{vhfffngghkjgghfjjghghghhjghjgghkghggSYSPNkljljjkklmjkjkjjk,mmm,kl;k;lklkaa}
indeed holds. Moreover, as before in the case of Liouville's
equation, we deduce that the equation
\er{vhfffngghkjgghfjjghghghhjghjgghkghggSYSPNkljljjkklmjkjkjjk,mmm,kl;k;lklkaa}
is invariant under the change of non-inertial cartesian coordinate
system of the form \er{noninchgravortbstrjgghguittu2jjjk}, provided
that we have
\er{yuythfgfyftydtydtydtyddyyyhhddhhhredPPN111hgghjgintintrrZZHHlllkljjjk},
\er{noninchgravortbstrjgghguittu2intrrrZZHHjkjhjh} and
\er{yuythfgfyftydtydtydtyddyyyhhddhhhredPPN111hgghjgintintrrZZHHlll}.

Next, we still need to prove that if $w$ is given by
\er{vhfffngghkjgghfjjghghghSYSPNjljllkljjhkhkhklpl;lpoop} then the
vector field $\vec u(\vec x,t)$ is indeed the macroscopic (average)
velocity field of every particle that we can found near the point
$\vec x$ at the instant of time $t$. Indeed, we need to prove the
following:
\begin{equation}\label{guytyurtydftyiujhioyhuyhSYSPNlkjljklkl}
\vec u(\vec x,t):=\frac{1}{\mu_j(\vec x,t)}\vec p_j(\vec x,t)+\vec
v(\vec x,t)-\frac{\sigma_j}{cm_j}\vec A(\vec x,t)\quad\quad\forall
j\in\{1,\ldots n\},
\end{equation}
where
\begin{equation}\label{vhfffngghkjgghfjjghghghhjghjgghkghggSYSPNkljljjmjkjjjkkjlkjkllkljkmjhhhkghggklklnnmmmhjjh}
\mu_j(\vec x,t):=\int m_j \,w\left(\vec P_1,\ldots,\vec P_n,\vec
r_1,\ldots,\vec r_n,t\right)\delta(\vec x-\vec r_j)\,d \vec
P_1,\ldots,d\vec P_n,d\vec r_1,\ldots,d\vec r_n,
\end{equation}
and
\begin{equation}\label{vhfffngghkjgghfjjghghghhjghjgghkghggSYSPNkljljjmjkjjjkkjlkjkllkljkmjhhhkghggklklnnmmmjhkhk}
\vec p_j(\vec x,t):=\int \vec P_j \,w\left(\vec P_1,\ldots,\vec
P_n,\vec r_1,\ldots,\vec r_n,t\right)\delta(\vec x-\vec r_j)\,d \vec
P_1,\ldots,d\vec P_n,d\vec r_1,\ldots,d\vec r_n.
\end{equation}
On the other hand, by \er{vhfffngghkjgghfjjghghghSYSPNjljll'l} we
have:
\begin{multline}\label{vhfffngghkjgghfjjghghghSYSPNjljll'lklkjjjkjjkj}
H_{\vec u}\left(\vec P_1,\ldots,\vec P_n,\vec r_1,\ldots,\vec
r_n,t\right)=
\sum_{j=1}^{n}\frac{1}{2m_j}\left|\vec P_j-\frac{\sigma_j}{c}\vec
A(\vec r_j,t)\right|^2\\+\sum_{j=1}^{n}\left(\vec
P_j-\frac{\sigma_j}{c}\vec A(\vec r_j,t)\right)\cdot\left(\vec
v(\vec r_j,t)-\vec u(\vec
r_j,t)\right)+\sum_{j=1}^{n}\frac{\sigma_j}{c}\vec A(\vec
r_j,t)\cdot\left(\vec v(\vec r_j,t)-\vec u(\vec r_j,t)\right)\\+
\sum_{j=1}^{n}\sigma_j\left(\Psi(\vec r_j,t)-\frac{1}{c}\vec A(\vec
r_j,t)\cdot\vec v(\vec r_j,t)\right)-V\left(\vec r_1,\ldots,\vec
r_n,t\right)\\
=
\sum_{j=1}^{n}\frac{1}{2m_j}\left|\vec P_j+m_j\vec v(\vec
r_j,t)-\frac{\sigma_j}{c}\vec A(\vec r_j,t)-m_j\vec u(\vec
r_j,t)\right|^2-\sum_{j=1}^{n}\frac{m_j}{2}\left|\vec v(\vec
r_j,t)-\vec u(\vec
r_j,t)\right|^2\\+\sum_{j=1}^{n}\frac{\sigma_j}{c}\vec A(\vec
r_j,t)\cdot\left(\vec v(\vec r_j,t)-\vec u(\vec r_j,t)\right)+
\sum_{j=1}^{n}\sigma_j\left(\Psi(\vec r_j,t)-\frac{1}{c}\vec A(\vec
r_j,t)\cdot\vec v(\vec r_j,t)\right)-V\left(\vec r_1,\ldots,\vec
r_n,t\right).
\end{multline}
Thus by inserting
\er{vhfffngghkjgghfjjghghghSYSPNjljll'lklkjjjkjjkj} into
\er{vhfffngghkjgghfjjghghghSYSPNjljllkljjhkhkhklpl;lpoop} we deduce
that \er{guytyurtydftyiujhioyhuyhSYSPNlkjljklkl} indeed holds.

Finally, by \er{vhfffngghkjgghfjjghghghSYSPNjljll'l} we have,
\begin{multline}\label{vhfffngghkjgghfjjghghghSYSPNjljll'lhkhh}
H_{\vec u}\left(\vec P_1,\ldots,\vec P_n,\vec r_1,\ldots,\vec
r_n,t\right)=
\sum_{j=1}^{n}\frac{1}{2m_j}\left|\vec
P_j\right|^2-\sum_{j=1}^{n}\vec P_j\cdot\left(\vec u(\vec
r_j,t)-\vec v(\vec r_j,t)+\frac{\sigma_j}{m_jc}\vec A(\vec
r_j,t)\right)\\+
\sum_{j=1}^{n}\left(\frac{\sigma^2_j}{2m_jc^2}\left|\vec A(\vec
r_j,t)\right|^2+\sigma_j\Psi(\vec r_j,t)-\frac{\sigma_j}{c}\vec
A(\vec r_j,t)\cdot\vec v(\vec r_j,t)\right)-V\left(\vec
r_1,\ldots,\vec r_n,t\right),
\end{multline}
that we rewrite as
\begin{equation}\label{vhfffngghkjgghfjjghghghSYSPNjljll'lhkhhhhh}
H_{\vec u}\left(\vec P_1,\ldots,\vec P_n,\vec r_1,\ldots,\vec
r_n,t\right)=
\sum_{j=1}^{n}\frac{1}{2m_j}\left|\vec
P_j\right|^2-\sum_{j=1}^{n}\vec P_j\cdot\vec h_j(\vec r_j,t)+
U\left(\vec r_1,\ldots,\vec r_n,t\right),
\end{equation}
where by $\vec h_j(\vec x,t)$ we denote a proper vector field,
defined as
\begin{equation}\label{vhfffngghkjgghfjjghghghhjghjgghkghggkghghjghSYSPNjjjhhhjjkjkjklkjljljkkhhkjkjkjkjjkjk}
\vec h_j(\vec x,t):= \left(\vec u(\vec x,t)-\vec v(\vec
x,t)+\frac{\sigma_j}{m_jc}\vec A(\vec x,t)\right),
\end{equation}
and by $U$ we denote a proper scalar, defined as
\begin{multline}\label{vhfffngghkjgghfjjghghghSYSPNjljll'lhkhhjgggjgjgghg}
U\left(\vec x_1,\ldots,\vec x_n,t\right):=\\
\sum_{j=1}^{n}\left(\frac{\sigma^2_j}{2m_jc^2}\left|\vec A(\vec
x_j,t)\right|^2+\sigma_j\Psi(\vec x_j,t)-\frac{\sigma_j}{c}\vec
A(\vec x_j,t)\cdot\vec v(\vec x_j,t)\right)-V\left(\vec
x_1,\ldots,\vec x_n,t\right).
\end{multline}

Thus,
\begin{equation}\label{vhfffngghkjgghfjjghghghSYSPNjljll'lhkhhhhhlll;l;iojj}
\frac{\partial H_{\vec u}}{\partial t}=
-\sum_{j=1}^{n}\vec P_j\cdot\frac{\partial \vec h_j}{\partial t}\vec
(\vec r_j,t)+ \frac{\partial U}{\partial t}\left(\vec
r_1,\ldots,\vec r_n,t\right),
\end{equation}
\begin{equation}\label{vhfffngghkjgghfjjghghghSYSPNjljll'lhkhhhhhll;l;uhhihj}
\nabla_{\vec r_j}H_{\vec u}=
-\left\{d_{\vec x}\vec h_j(\vec r_j,t)\right\}^T\cdot\vec P_j+
\nabla_{\vec r_j}U\left(\vec r_1,\ldots,\vec r_n,t\right),
\end{equation}
and
\begin{equation}\label{vhfffngghkjgghfjjghghghSYSPNjljll'lhkhhhhhl;;kkkk;k}
\nabla_{\vec P_j}H_{\vec u}=
\frac{1}{m_j}\vec P_j-\vec h_j(\vec r_j,t).
\end{equation}
On the other hand, by
\er{vhfffngghkjgghfjjghghghhjghjgghkghggSYSPNkljljjkklmjkjkjjk,mmm,kl;k;lklk}
we have
\begin{equation}\label{vhfffngghkjgghfjjghghghhjghjgghkghggSYSPNkljljjkklmjkjkjjk,mmm,kl;k;lklkjkjkk1}
\frac{\partial H_{\vec u}}{\partial t}+\sum_{j=1}^{n}\vec u( \vec
r_j,t)\cdot\nabla_{\vec r_j}H_{\vec
u}-\sum_{j=1}^{n}\left(\left\{d_{\vec x}\vec u(\vec
r_j,t)\right\}^T\cdot\vec P_j\right)\cdot\nabla_{\vec P_j}H_{\vec
u}=0.
\end{equation}
Thus inserting
\er{vhfffngghkjgghfjjghghghSYSPNjljll'lhkhhhhhlll;l;iojj},
\er{vhfffngghkjgghfjjghghghSYSPNjljll'lhkhhhhhll;l;uhhihj} and
\er{vhfffngghkjgghfjjghghghSYSPNjljll'lhkhhhhhl;;kkkk;k} into
\er{vhfffngghkjgghfjjghghghhjghjgghkghggSYSPNkljljjkklmjkjkjjk,mmm,kl;k;lklkjkjkk1}
gives
\begin{multline}\label{vhfffngghkjgghfjjghghghhjghjgghkghggSYSPNkljljjkklmjkjkjjk,mmm,kl;k;lklkjkjkkjkjkjkjlj}
-\sum_{j=1}^{n}\vec P_j\cdot\frac{\partial \vec h_j}{\partial t}\vec
(\vec r_j,t)+ \frac{\partial U}{\partial t}\left(\vec
r_1,\ldots,\vec r_n,t\right)+\sum_{j=1}^{n}\vec u( \vec
r_j,t)\cdot\left(-\left\{d_{\vec x}\vec h_j(\vec
r_j,t)\right\}^T\cdot\vec P_j+ \nabla_{\vec r_j}U\left(\vec
r_1,\ldots,\vec
r_n,t\right)\right)\\-\sum_{j=1}^{n}\left(\left\{d_{\vec x}\vec
u(\vec r_j,t)\right\}^T\cdot\vec
P_j\right)\cdot\left(\frac{1}{m_j}\vec P_j-\vec h_j(\vec
r_j,t)\right)=0.
\end{multline}
I.e.
\begin{multline}\label{vhfffngghkjgghfjjghghghhjghjgghkghggSYSPNkljljjkklmjkjkjjk,mmm,kl;k;lklkjkjkkjkjkjkjljljjljkl}
\left(\frac{\partial U}{\partial t}\left(\vec r_1,\ldots,\vec
r_n,t\right)+\sum_{j=1}^{n}\vec u( \vec r_j,t)\cdot\nabla_{\vec
r_j}U\left(\vec r_1,\ldots,\vec
r_n,t\right)\right)\\-\sum_{j=1}^{n}\vec
P_j\cdot\left(\frac{\partial \vec h_j}{\partial t}\vec (\vec
r_j,t)+d_{\vec x}\vec h_j(\vec r_j,t)\cdot\vec u( \vec
r_j,t)-d_{\vec x}\vec u(\vec r_j,t)\cdot\vec h_j( \vec
r_j,t)\right)\\-\sum_{j=1}^{n}\frac{1}{2m_j}\left(\left(d_{\vec
x}\vec u(\vec r_j,t)+\left\{d_{\vec x}\vec u(\vec
r_j,t)\right\}^T\right)\cdot\vec P_j\right)\cdot\vec P_j=0.
\end{multline}
Thus by \er{apfrm6} we rewrite
\er{vhfffngghkjgghfjjghghghhjghjgghkghggSYSPNkljljjkklmjkjkjjk,mmm,kl;k;lklkjkjkkjkjkjkjljljjljkl}
as
\begin{multline}\label{vhfffngghkjgghfjjghghghhjghjgghkghggSYSPNkljljjkklmjkjkjjk,mmm,kl;k;lklkjkjkkjkjkjkjljljjljklggghgkhhkhkhk}
\left(\frac{\partial U}{\partial t}\left(\vec r_1,\ldots,\vec
r_n,t\right)+\sum_{j=1}^{n}\vec u( \vec r_j,t)\cdot\nabla_{\vec
r_j}U\left(\vec r_1,\ldots,\vec
r_n,t\right)\right)\\-\sum_{j=1}^{n}\vec
P_j\cdot\left(\frac{\partial \vec h_j}{\partial t}\vec (\vec
r_j,t)-curl_{\vec x}\left(\vec u(\vec r_j,t)\times\vec h_j(\vec
r_j,t)\right)+\left(\Div_{\vec x}\vec h_j(\vec r_j,t)\right)\vec
u(\vec r_j,t)-\left(\Div_{\vec x}\vec u(\vec r_j,t)\right)\vec
h_j(\vec
r_j,t)\right)\\-\sum_{j=1}^{n}\frac{1}{2m_j}\left(\left(d_{\vec
x}\vec u(\vec r_j,t)+\left\{d_{\vec x}\vec u(\vec
r_j,t)\right\}^T\right)\cdot\vec P_j\right)\cdot\vec P_j=0,
\end{multline}
that is equivalent to the following
\begin{equation}\label{vhfffngghkjgghfjjghghghhjghjgghkghggSYSPNkljljjkklmjkjkjjk,mmm,kl;k;lklkjkjkkjkjkjkjljljjljklggghgkhhkhkhkkklk}
\begin{cases}
\frac{\partial U}{\partial t}\left(\vec x_1,\ldots,\vec
x_n,t\right)+\sum_{j=1}^{n}\vec u( \vec x_j,t)\cdot\nabla_{\vec
x_j}U\left(\vec x_1,\ldots,\vec x_n,t\right)=0\\
\frac{\partial \vec h_j}{\partial t}\vec (\vec x,t)-curl_{\vec
x}\left(\vec u(\vec x,t)\times\vec h_j(\vec
x,t)\right)+\left(\Div_{\vec x}\vec h_j(\vec x,t)\right)\vec u(\vec
x,t)=0\quad\quad\quad\forall j=1,\ldots,n\\d_{\vec x}\vec u(\vec
x,t)+\left\{d_{\vec x}\vec u(\vec x,t)\right\}^T=0.
\end{cases}
\end{equation}
Next, using
\er{vhfffngghkjgghfjjghghghhjghjgghkghggkghghjghSYSPNjjjhhhjjkjkjklkjljljkkhhkjkjkjkjjkjk}
we rewrite
\er{vhfffngghkjgghfjjghghghhjghjgghkghggSYSPNkljljjkklmjkjkjjk,mmm,kl;k;lklkjkjkkjkjkjkjljljjljklggghgkhhkhkhkkklk}
as:
\begin{equation}\label{vhfffngghkjgghfjjghghghhjghjgghkghggSYSPNkljljjkklmjkjkjjk,mmm,kl;k;lklkjkjkkjkjkjkjljljjljklggghgkhhkhkhkkklkklklkl}
\begin{cases}
\frac{\partial U}{\partial t}\left(\vec x_1,\ldots,\vec
x_n,t\right)+\sum_{j=1}^{n}\vec u( \vec x_j,t)\cdot\nabla_{\vec
x_j}U\left(\vec x_1,\ldots,\vec x_n,t\right)=0,\\
\frac{\partial}{\partial t}\left(\vec u(\vec x,t)-\vec v(\vec
x,t)+\frac{\sigma_j}{m_jc}\vec A(\vec x,t)\right)-curl_{\vec
x}\left(\vec u(\vec x,t)\times\left(\vec u(\vec x,t)-\vec v(\vec
x,t)+\frac{\sigma_j}{m_jc}\vec A(\vec
x,t)\right)\right)\\+\left(\Div_{\vec x}\left(\vec u(\vec x,t)-\vec
v(\vec x,t)+\frac{\sigma_j}{m_jc}\vec A(\vec x,t)\right)\right)\vec
u(\vec x,t)=0\quad\quad\quad\quad\forall j=1,\ldots,n\,\\d_{\vec
x}\vec u(\vec x,t)+\left\{d_{\vec x}\vec u(\vec x,t)\right\}^T=0.
\end{cases}
\end{equation}
Thus since the second equation in
\er{vhfffngghkjgghfjjghghghhjghjgghkghggSYSPNkljljjkklmjkjkjjk,mmm,kl;k;lklkjkjkkjkjkjkjljljjljklggghgkhhkhkhkkklkklklkl}
must be valid for every $j=1,\ldots,n$ by
\er{vhfffngghkjgghfjjghghghhjghjgghkghggSYSPNkljljjkklmjkjkjjk,mmm,kl;k;lklkjkjkkjkjkjkjljljjljklggghgkhhkhkhkkklkklklkl}
we deduce
\begin{equation}\label{vhfffngghkjgghfjjghghghhjghjgghkghggSYSPNkljljjkklmjkjkjjk,mmm,kl;k;lklkjkjkkjkjkjkjljljjljklggghgkhhkhkhkkklkklklkljkhkhjjkhjh}
\begin{cases}
\frac{\partial U}{\partial t}\left(\vec x_1,\ldots,\vec
x_n,t\right)+\sum_{j=1}^{n}\vec u( \vec x_j,t)\cdot\nabla_{\vec
x_j}U\left(\vec x_1,\ldots,\vec x_n,t\right)=0,\\
\frac{\partial}{\partial t}\left(\vec u(\vec x,t)-\vec v(\vec
x,t)\right)-curl_{\vec x}\left(\vec u(\vec x,t)\times\left(\vec
u(\vec x,t)-\vec v(\vec x,t)\right)\right)+\left(\Div_{\vec
x}\left(\vec u(\vec x,t)-\vec v(\vec x,t)\right)\right)\vec
u(\vec x,t)=0,\\
\frac{\partial\vec A}{\partial t}(\vec x,t)-curl_{\vec x}\left(\vec
u(\vec x,t)\times\vec A(\vec x,t)\right)+\left(\Div_{\vec x}\vec
A(\vec x,t)\right)\vec u(\vec x,t)=0,
\\
d_{\vec x}\vec u(\vec x,t)+\left\{d_{\vec x}\vec u(\vec
x,t)\right\}^T=0,
\end{cases}
\end{equation}
where
\begin{multline}\label{vhfffngghkjgghfjjghghghSYSPNjljll'lhkhhjgggjgjgghgjkjjk}
U\left(\vec x_1,\ldots,\vec x_n,t\right):=\\
\sum_{j=1}^{n}\left(\frac{\sigma^2_j}{2m_jc^2}\left|\vec A(\vec
x_j,t)\right|^2+\sigma_j\Psi(\vec x_j,t)-\frac{\sigma_j}{c}\vec
A(\vec x_j,t)\cdot\vec v(\vec x_j,t)\right)-V\left(\vec
x_1,\ldots,\vec x_n,t\right).
\end{multline}
Then, clearly equations
\er{vhfffngghkjgghfjjghghghhjghjgghkghggSYSPNkljljjkklmjkjkjjk,mmm,kl;k;lklkjkjkkjkjkjkjljljjljklggghgkhhkhkhkkklkklklkl}
and
\er{vhfffngghkjgghfjjghghghhjghjgghkghggSYSPNkljljjkklmjkjkjjk,mmm,kl;k;lklkjkjkkjkjkjkjljljjljklggghgkhhkhkhkkklkklklkljkhkhjjkhjh},
\er{vhfffngghkjgghfjjghghghSYSPNjljll'lhkhhjgggjgjgghgjkjjk} are
invariant under the change of inertial or non-inertial cartesian
coordinate systems. Moreover,
\er{vhfffngghkjgghfjjghghghhjghjgghkghggSYSPNkljljjkklmjkjkjjk,mmm,kl;k;lklkjkjkkjkjkjkjljljjljklggghgkhhkhkhkkklkklklkljkhkhjjkhjh},
\er{vhfffngghkjgghfjjghghghSYSPNjljll'lhkhhjgggjgjgghgjkjjk} are
equivalent to
\er{vhfffngghkjgghfjjghghghhjghjgghkghggSYSPNkljljjkklmjkjkjjk,mmm,kl;k;lklk}
and then represent the necessary conditions on the field $\vec u$
and the system for the existence of a possibility to achieve the
Thermodynamical equilibrium in the given system.

Next note that by
\er{vhfffngghkjgghfjjghghghhjghjgghkghggSYSPNkljljjkklmjkjkjjk,mmm,kl;k;lklkjkjkkjkjkjkjljljjljklggghgkhhkhkhkkklkklklkljkhkhjjkhjh}
and Proposition \ref{jhjkhhjhjhgjgjjkjk} there exists another
cartesian coordinate system $(**)$ such that under the change of
coordinate system $(*)$ to another cartesian coordinate system
$(**)$, given by \er{noninchgravortbstrjgghguittu2hhhhh}, we have
\begin{equation}
\label{jljjjjkhhhjkkbb}
A(t)\cdot \vec u(\vec x,t)+
\frac{d A}{dt}(t)\cdot\vec x+
\frac{d\vec z}{dt}(t)=\vec u'(\vec x',t')=0.
\end{equation}
Thus, since
\er{vhfffngghkjgghfjjghghghhjghjgghkghggSYSPNkljljjkklmjkjkjjk,mmm,kl;k;lklkjkjkkjkjkjkjljljjljklggghgkhhkhkhkkklkklklkljkhkhjjkhjh},
\er{vhfffngghkjgghfjjghghghSYSPNjljll'lhkhhjgggjgjgghgjkjjk} are
invariant under the change of inertial or non-inertial cartesian
coordinate systems, in system $(**)$ we have
\begin{equation}\label{vhfffngghkjgghfjjghghghhjghjgghkghggSYSPNkljljjkklmjkjkjjk,mmm,kl;k;lklkjkjkkjkjkjkjljljjljklggghgkhhkhkhkkklkklklkljkhkhjjkhjhmjmm,mjhjkjk}
\begin{cases}
\frac{\partial U'}{\partial t'}\left(\vec x'_1,\ldots,\vec
x'_n,t'\right)=0,\\
\frac{\partial\vec v'}{\partial t'}(\vec
x',t')=0,\\
\frac{\partial\vec A'}{\partial t'}(\vec x',t')=0
\\
\vec u'(\vec x',t')=0,
\end{cases}
\end{equation}
where
\begin{multline}\label{vhfffngghkjgghfjjghghghSYSPNjljll'lhkhhjgggjgjgghgjkjjkkj}
U'\left(\vec x'_1,\ldots,\vec x'_n,t'\right):=\\
\sum_{j=1}^{n}\left(\frac{(\sigma'_j)^2}{2m'_jc^2}\left|\vec A'(\vec
x'_j,t')\right|^2+\sigma'_j\Psi'(\vec
x'_j,t')-\frac{\sigma'_j}{c}\vec A'(\vec x'_j,t')\cdot\vec v'(\vec
x'_j,t')\right)-V'\left(\vec x'_1,\ldots,\vec x'_n,t'\right).
\end{multline}
On the other hand,
\er{vhfffngghkjgghfjjghghghhjghjgghkghggSYSPNkljljjkklmjkjkjjk,mmm,kl;k;lklkjkjkkjkjkjkjljljjljklggghgkhhkhkhkkklkklklkljkhkhjjkhjhmjmm,mjhjkjk}
is equivalent to
\begin{equation}\label{vhfffnhhjhjjkkjhk}
\frac{\partial H'_{0}}{\partial t'}\left(\vec P'_1,\ldots,\vec
P'_n,\vec r'_1,\ldots,\vec r'_n,t'\right)=0\quad\text{and}\quad\vec
u'(\vec x',t')=0.
\end{equation}
Thus we obtain that the condition for existing the possibility to
achieve the Thermodynamical equilibrium in the given system, ruled
by the Hamiltonian $H_0\left(\vec P_1,\ldots,\vec P_n,\vec
r_1,\ldots,\vec r_n,t\right)$, of the form
\er{vhfffngghkjgghfjjghghghhjghjgghkghggSYSPNkljljjkklmjkjkjjk,mmm,kl;k;lklkjkjkk1},
is equivalent to the existence of a cartesian coordinate system
$(**)$, where the average velocity of every particle vanishes and at
the same time in the system $(**)$ the Hamiltonian is independent of
the time variable explicitly.

%
%
%
%
%
%

\subsection{Shr\"{o}dinger equation for a finite system of quantum
particles}\label{hggyugyuy1} Consider the motion of a system of $n$
quantum micro-particles with inertial masses $m_1,\ldots,m_n$ and
the charges $\sigma_1,\ldots,\sigma_n$ in the outer gravitational
and electromagnetical field with characteristics $\vec v(\vec x,t)$,
$\vec A(\vec x,t)$ and $\Psi(\vec x,t)$ and additional conservative
field with potential $V(\vec y_1,\ldots,\vec y_n,t)$, not taking
into account the spin interaction. The Shr\"{o}dinger equation for
this system of particles is
\begin{equation}\label{vhfffngghkjgghfjjghghghhjghjgghkghggkghghjghSYSPN}
i\hbar\frac{\partial\psi}{\partial t}=\hat H_0\cdot\psi,
\end{equation}
where $\psi(\vec x_1,\ldots,\vec x_n,t)\in\mathbb{C}$ is a wave
function and $\hat H_0$ is the Hamiltonian operator. Since by
\er{vhfffngghkjgghfjjghghghSYSPN} the Hamiltonian for a
macro-particles has the form
\begin{multline}\label{vhfffngghkjgghfjjghghghhjghjgghkghggSYSPN}
H_{\text{macro}}\left(\vec P_1,\ldots,\vec P_n,\vec r_1,\ldots,\vec
r_n,t\right)=-V\left(\vec r_1,\ldots,\vec r_n,t\right)+
\sum_{j=1}^{n}\frac{1}{2}\vec P_j\cdot\vec v(\vec
r_j,t)+\sum_{j=1}^{n}\frac{1}{2}\vec v(\vec r_j,t)\cdot\vec P_j\\+
\sum_{j=1}^{n}\frac{1}{2m_j}\left|\vec P_j-\frac{\sigma_j}{c}\vec
A(\vec r_j,t)\right|^2+\sum_{j=1}^{n}\sigma_j\left(\Psi(\vec
r_j,t)-\frac{1}{c}\vec A(\vec r_j,t)\cdot\vec v(\vec r_j,t)\right),
\end{multline}
we built the Hermitian Hamiltonian operator as
\begin{multline}\label{vhfffngghkjgghfjjghghghhjghjgghkghgghjhggjjkgSYSPN}
\hat H_0\cdot\psi=-\sum_{j=1}^{n}\frac{i\hbar}{2}div_{\vec
x_j}\left\{\psi\vec v(\vec
x_j,t)\right\}-\sum_{j=1}^{n}\frac{i\hbar}{2}\vec v(\vec
x_j,t)\cdot\nabla_{\vec x_j}\psi\\+
\sum_{j=1}^{n}\left\{\frac{1}{2m_j}\left(-i\hbar\nabla_{\vec
x_j}-\frac{\sigma_j}{c}\vec A(\vec
x_j,t)\right)\circ\left(-i\hbar\nabla_{\vec
x_j}-\frac{\sigma_j}{c}\vec A(\vec
x_j,t)\right)\right\}\cdot\psi\\+\sum_{j=1}^{n}\sigma_j\left(\Psi(\vec
x_j,t)-\frac{1}{c}\vec A(\vec x_j,t)\cdot\vec v(\vec
x_j,t)\right)\cdot\psi-V\left(\vec x_1,\ldots,\vec
x_n,t\right)\cdot\psi=
-\sum_{j=1}^{n}\frac{\hbar^2}{2m_j}\Delta_{\vec
x_j}\psi-\sum_{j=1}^{n}\frac{i\hbar}{2}div_{\vec x_j}\left\{\psi\vec
v(\vec x_j,t)\right\}\\-\sum_{j=1}^{n}\frac{i\hbar}{2}\vec v(\vec
x_j,t)\cdot\nabla_{\vec x_j}\psi+\sum_{j=1}^{n}\frac{
i\hbar\sigma_j}{2m_jc}div_{\vec x_j}\left\{\psi\vec A(\vec
x_j,t)\right\}+\sum_{j=1}^{n}\frac{ i\hbar\sigma_j}{2m_jc}\vec
A(\vec x_j,t)\cdot\nabla_{\vec
x_j}\psi\\+\sum_{j=1}^{n}\left(\sigma_j\Psi(\vec
x_j,t)-\frac{\sigma_j}{c}\vec A(\vec x_j,t)\cdot\vec v(\vec
x_j,t)+\frac{\sigma^2_j}{2m_jc^2}\left|\vec A(\vec
x_j,t)\right|^2\right)\psi-V\left(\vec x_1,\ldots,\vec
x_n,t\right)\psi,
\end{multline}
Thus the corresponding Shr\"{o}dinger equation will be
\begin{multline}\label{vhfffngghkjgghfjjghghghhjghjgghkghgghjhggjjkgfgdSYSPN}
i\hbar\frac{\partial\psi}{\partial t}=\hat H_0\cdot\psi=
-\sum_{j=1}^{n}\frac{\hbar^2}{2m_j}\Delta_{\vec
x_j}\psi-\sum_{j=1}^{n}\frac{i\hbar}{2}div_{\vec x_j}\left\{\psi\vec
v(\vec x_j,t)\right\}-\sum_{j=1}^{n}\frac{i\hbar}{2}\vec v(\vec
x_j,t)\cdot\nabla_{\vec x_j}\psi\\+\sum_{j=1}^{n}\frac{
i\hbar\sigma_j}{2m_jc}div_{\vec x_j}\left\{\psi\vec A(\vec
x_j,t)\right\}+\sum_{j=1}^{n}\frac{ i\hbar\sigma_j}{2m_jc}\vec
A(\vec x_j,t)\cdot\nabla_{\vec
x_j}\psi\\+\sum_{j=1}^{n}\left(\sigma_j\Psi(\vec
x_j,t)-\frac{\sigma_j}{c}\vec A(\vec x_j,t)\cdot\vec v(\vec
x_j,t)+\frac{\sigma^2_j}{2m_jc^2}\left|\vec A(\vec
x_j,t)\right|^2\right)\psi-V\left(\vec x_1,\ldots,\vec
x_n,t\right)\psi.
\end{multline}
So
%
%
%
%
%
%
\begin{multline}\label{vhfffngghkjgghfjjghghghhjghjgghkghgghjhggjjkgfgdiyfgfjkjgjgSYSPN}
i\hbar\left(\frac{\partial\psi}{\partial t}+\sum_{j=1}^{n}\vec
v(\vec x_j,t)\cdot\nabla_{\vec
x_j}\psi\right)+\sum_{j=1}^{n}\frac{i\hbar}{2}\left(div_{\vec
x_j}\vec v(\vec x_j,t)\right)\psi=
-\sum_{j=1}^{n}\frac{\hbar^2}{2m_j}\Delta_{\vec x_j}\psi-V\left(\vec
x_1,\ldots,\vec x_n,t\right)\psi\\+\sum_{j=1}^{n}\frac{
i\hbar\sigma_j}{2m_jc}div_{\vec x_j}\left\{\psi\vec A(\vec
x_j,t)\right\}+\sum_{j=1}^{n}\frac{ i\hbar\sigma_j}{2m_jc}\vec
A(\vec x_j,t)\cdot\nabla_{\vec
x_j}\psi\\+\sum_{j=1}^{n}\left(\sigma_j\Psi(\vec
x_j,t)-\frac{\sigma_j}{c}\vec A(\vec x_j,t)\cdot\vec v(\vec
x_j,t)+\frac{\sigma^2_j}{2m_jc^2}\left|\vec A(\vec
x_j,t)\right|^2\right)\psi.
\end{multline}
Next consider a change of some non-inertial cartesian coordinate
system $(*)$ to another cartesian coordinate system $(**)$ of the
form:
\begin{equation}\label{noninchgravortbstrghgggSYSPN}
\begin{cases}
\vec x'=A(t)\cdot\vec x+\vec z(t),\\
t'=t,
\end{cases}
\end{equation}
where $A(t)\in SO(3)$ is a rotation, i.e. $A(t)\in \R^{3\times 3}$,
$det\, A(t)>0$ and $A(t)\cdot A^T(t)=I$. Then, since
\begin{equation}\label{vyfgjhgjhvhgghSYSPN}
\begin{cases}
\psi'=\psi
\\
V'=V
\\
\vec A'=A(t)\cdot\vec A
\\
\Psi'-\frac{1}{c}\vec A'\cdot\vec v'=\Psi-\frac{1}{c}\vec A\cdot\vec
v,
\end{cases}
\end{equation}
we deduce that  the Shr\"{o}dinger equation of the form
\er{vhfffngghkjgghfjjghghghhjghjgghkghgghjhggjjkgfgdiyfgfjkjgjgSYSPN}
is invariant under the change of non-inertial cartesian coordinate
system. So the quantum mechanical laws are also invariant in every
non-inertial cartesian coordinate system.

Next, assume that in inertial coordinate system $(*)$ we have:
\begin{equation}
\label{MaxVacFull1ninshtrgravortghhghgjkgghklhjgkghghjjkjhjkkggjkhjkhjjhhfhjhklkhkhjjklzzzyyyhjggjhgghhjhNWNWBWHWNWSYSPN}
\begin{cases}
curl_{\vec x}\vec v= 0,\\
\frac{\partial\vec v}{\partial t}+d_\vec x\vec v\cdot\vec v=
-\nabla_{\vec x}\Phi,
\end{cases}
\end{equation}
%
%
%
%
%
%
where $\Phi$ is the scalar gravitational potential. Since in the
system $(*)$ we have $curl_{\vec x}\vec v=0$ we can rewrite
\er{MaxVacFull1ninshtrgravortghhghgjkgghklhjgkghghjjkjhjkkggjkhjkhjjhhfhjhklkhkhjjklzzzyyyhjggjhgghhjhNWNWBWHWNWSYSPN}
as
\begin{equation}
\label{MaxVacFull1ninshtrgravortghhghgjkgghklhjgkghghjjkjhjkkggjkhjkhjjhhfhjhklkhkhjjklzzzyyyhjggjhgghhjhNWNWNWNWNWBWHWNWSYSPN}
\begin{cases}
\vec v=\nabla_{\vec x}Z,\\
\frac{\partial Z}{\partial t}+\frac{1}{2}\left|\nabla_{\vec
x}Z\right|^2=-\Phi.
\end{cases}
\end{equation}
Thus by
\er{MaxVacFull1ninshtrgravortghhghgjkgghklhjgkghghjjkjhjkkggjkhjkhjjhhfhjhklkhkhjjklzzzyyyhjggjhgghhjhNWNWNWNWNWBWHWNWSYSPN}
we rewrite
\er{vhfffngghkjgghfjjghghghhjghjgghkghgghjhggjjkgfgdiyfgfjkjgjgSYSPN}
as
\begin{multline}\label{vhfffngghkjgghfjjghghghhjghjgghkghgghjhggjjkgfgdiyfgfjkjgjggjjgSYSPN}
i\hbar\frac{\partial\psi}{\partial
t}+\sum_{j=1}^{n}i\hbar\nabla_{\vec x_j}Z(\vec
x_j,t)\cdot\nabla_{\vec
x_j}\psi+\sum_{j=1}^{n}\frac{i\hbar}{2}\left(\Delta_{\vec x_j}Z(\vec
x_j,t)\right)\psi +\sum_{j=1}^{n}\frac{\hbar^2}{2m_j}\Delta_{\vec
x_j}\psi=\\
\sum_{j=1}^{n}\frac{ i\hbar\sigma_j}{2m_jc}\left(div_{\vec x_j}\vec
A(\vec x_j,t)\right)\psi
+\sum_{j=1}^{n}\frac{ i\hbar\sigma_j}{m_jc}\vec A(\vec
x_j,t)\cdot\nabla_{\vec
x_j}\psi-\sum_{j=1}^{n}\frac{\sigma_j}{c}\left(\vec A(\vec
x_j,t)\cdot\nabla_{\vec x_j}Z(\vec
x_j,t)\right)\psi\\+\sum_{j=1}^{n}\left(\sigma_j\Psi(\vec
x_j,t)+\frac{\sigma^2_j}{2m_jc^2}\left|\vec A(\vec
x_j,t)\right|^2\right)\psi-V\psi.
\end{multline}
Then multiplying
\er{vhfffngghkjgghfjjghghghhjghjgghkghgghjhggjjkgfgdiyfgfjkjgjggjjgSYSPN}
by factor $e^{\sum_{j=1}^{n}\frac{im_j}{\hbar}Z(\vec x_j,t)}$ gives:
\begin{multline}\label{vhfffngghkjgghfjjghghghhjghjgghkghgghjhggjjkgfgdiyfgfjkjgjggjjgugyyjSYSPN}
i\hbar\frac{\partial\psi}{\partial
t}e^{\sum_{j=1}^{n}\frac{im_j}{\hbar}Z(\vec
x_j,t)}+\sum_{k=1}^{n}i\hbar\left(\nabla_{\vec x_k}Z(\vec
x_k,t)\cdot\nabla_{\vec
x_k}\psi\right)e^{\sum_{j=1}^{n}\frac{im_j}{\hbar}Z(\vec
x_j,t)}\\+\sum_{k=1}^{n}\frac{i\hbar}{2}\left(\Delta_{\vec
x_k}Z(\vec x_k,t)\right)e^{\sum_{j=1}^{n}\frac{im_j}{\hbar}Z(\vec
x_j,t)}\psi
 +\sum_{k=1}^{n}\frac{\hbar^2}{2m_k}\left(\Delta_{\vec x_k}\psi\right)e^{\sum_{j=1}^{n}\frac{im_j}{\hbar}Z(\vec x_j,t)}=\\ \sum_{k=1}^{n}\frac{ i\hbar\sigma_k}{2m_kc}\left(div_{\vec x_k}\vec
A(\vec x_k,t)\right)e^{\sum_{j=1}^{n}\frac{im_j}{\hbar}Z(\vec
x_j,t)}\psi +
\sum_{k=1}^{n}\frac{ i\hbar\sigma_k}{m_kc}\left(\vec A(\vec
x_k,t)\cdot\nabla_{\vec
x_k}\psi\right)e^{\sum_{j=1}^{n}\frac{im_j}{\hbar}Z(\vec
x_j,t)}\\-\sum_{k=1}^{n}\frac{\sigma_k}{c}\left(\vec A(\vec
x_k,t)\cdot\nabla_{\vec x_k}Z(\vec
x_k,t)\right)e^{\sum_{j=1}^{n}\frac{im_j}{\hbar}Z(\vec
x_j,t)}\psi-V\left(e^{\sum_{j=1}^{n}\frac{im_j}{\hbar}Z(\vec
x_j,t)}\psi\right)\\+\sum_{k=1}^{n}\left(\sigma_k\Psi(\vec
x_k,t)+\frac{\sigma^2_k}{2m_kc^2}\left|\vec A(\vec
x_k,t)\right|^2\right)\left(e^{\sum_{j=1}^{n}\frac{im_j}{\hbar}Z(\vec
x_j,t)}\psi\right).
\end{multline}
We rewrite
\er{vhfffngghkjgghfjjghghghhjghjgghkghgghjhggjjkgfgdiyfgfjkjgjggjjgugyyjSYSPN}
as
\begin{multline}\label{vhfffngghkjgghfjjghghghhjghjgghkghgghjhggjjkgfgdiyfgfjkjgjggjjgugyyjjkgghgjSYSPN}
i\hbar\frac{\partial}{\partial
t}\left(e^{\sum_{j=1}^{n}\frac{im_j}{\hbar}Z(\vec
x_j,t)}\psi\right)+\sum_{k=1}^{n}\frac{\hbar^2}{2m_k}\Delta_{\vec
x_k}\left(e^{\sum_{j=1}^{n}\frac{im_j}{\hbar}Z(\vec
x_j,t)}\psi\right)=\\ \sum_{k=1}^{n}\frac{
i\hbar\sigma_k}{2m_kc}\left(div_{\vec x_k}\vec A(\vec
x_k,t)\right)e^{\sum_{j=1}^{n}\frac{im_j}{\hbar}Z(\vec x_j,t)}\psi+
\sum_{k=1}^{n}\frac{ i\hbar\sigma_k}{m_kc}\vec A(\vec
x_k,t)\cdot\nabla_{\vec
x_k}\left(e^{\sum_{j=1}^{n}\frac{im_j}{\hbar}Z(\vec
x_j,t)}\psi\right)
\\+\sum_{k=1}^{n}\left(\sigma_k\Psi(\vec
x_k,t)+\frac{\sigma^2_k}{2m_kc^2}\left|\vec A(\vec
x_k,t)\right|^2\right)\left(e^{\sum_{j=1}^{n}\frac{im_j}{\hbar}Z(\vec
x_j,t)}\psi\right)-V\left(e^{\sum_{j=1}^{n}\frac{im_j}{\hbar}Z(\vec
x_j,t)}\psi\right)\\-\sum_{k=1}^{n}m_k\left(\frac{\partial
Z}{\partial t}(\vec x_k,t)+\frac{1}{2}\left|\nabla_{\vec x_k}Z(\vec
x_k,t)\right|^2\right)\left(e^{\sum_{j=1}^{n}\frac{im_j}{\hbar}Z(\vec
x_j,t)}\psi\right).
\end{multline}
Therefore, inserting
\er{MaxVacFull1ninshtrgravortghhghgjkgghklhjgkghghjjkjhjkkggjkhjkhjjhhfhjhklkhkhjjklzzzyyyhjggjhgghhjhNWNWNWNWNWBWHWNWSYSPN}
into
\er{vhfffngghkjgghfjjghghghhjghjgghkghgghjhggjjkgfgdiyfgfjkjgjggjjgugyyjjkgghgjSYSPN}
gives
\begin{multline}\label{vhfffngghkjgghfjjghghghhjghjgghkghgghjhggjjkgfgdiyfgfjkjgjggjjgugyyjjkgghgjjhhjkkSYSPN}
i\hbar\frac{\partial}{\partial
t}\left(e^{\sum_{j=1}^{n}\frac{im_j}{\hbar}Z(\vec
x_j,t)}\psi\right)=-\sum_{k=1}^{n}\frac{\hbar^2}{2m_k}\Delta_{\vec
x}\left(e^{\sum_{j=1}^{n}\frac{im_j}{\hbar}Z(\vec
x_j,t)}\psi\right)-V\left(e^{\sum_{j=1}^{n}\frac{im_j}{\hbar}Z(\vec
x_j,t)}\psi\right)
\\+\sum_{k=1}^{n}\frac{ i\hbar\sigma_k}{2m_kc}\left(div_{\vec x_k}\vec
A(\vec x_k,t)\right)e^{\sum_{j=1}^{n}\frac{im_j}{\hbar}Z(\vec
x_j,t)}\psi +\sum_{k=1}^{n}\frac{ i\hbar\sigma_k}{m_kc}\vec A(\vec
x_k,t)\cdot\nabla_{\vec
x_k}\left(e^{\sum_{j=1}^{n}\frac{im_j}{\hbar}Z(\vec
x_j,t)}\psi\right)\\
+\sum_{k=1}^{n}\left(\sigma_k\Psi(\vec
x_k,t)+\frac{\sigma^2_k}{2m_kc^2}\left|\vec A(\vec
x_k,t)\right|^2+m_k\Phi(\vec
x_k,t)\right)\left(e^{\sum_{j=1}^{n}\frac{im_j}{\hbar}Z(\vec
x_j,t)}\psi\right).
\end{multline}
Then denoting
\begin{equation}\label{jhgghfhj}
\psi_1:=e^{\sum_{j=1}^{n}\frac{im_j}{\hbar}Z(\vec x_j,t)}\psi,
\end{equation}
we obtain in the
coordinate system $(*)$ the Shr\"{o}dinger equation in the form
\begin{multline}\label{vhfffngghkjgghfjjghghghhjghjgghkghgghjhggjjkgfgdiyfgfjkjgjggjjgugyyjjkgghgjjhhjkkhjhjkSYSPN}
i\hbar\frac{\partial\psi_1}{\partial
t}=-\sum_{j=1}^{n}\frac{\hbar^2}{2m_j}\Delta_{\vec
x_j}\psi_1+\sum_{j=1}^{n}\frac{ i\hbar\sigma_j}{2m_jc}div_{\vec
x_j}\left\{\psi_1\vec A(\vec x_j,t)\right\}
+\sum_{j=1}^{n}\frac{ i\hbar\sigma_j}{2m_jc}\vec A(\vec
x_j,t)\cdot\nabla_{\vec x_j}\psi_1
\\+\sum_{j=1}^{n}\left(\sigma_j\Psi(\vec
x_j,t)+\frac{\sigma^2_j}{2m_jc^2}\left|\vec A(\vec
x_j,t)\right|^2+m_j\Phi(\vec x_j,t)\right)\psi_1-V\psi_1,
\end{multline}
which coincides with the classical Shr\"{o}dinger equation for this
case.
\begin{remark}\label{gyytghffg}
Note that by \er{noninchgravortbstrjgghguittu1intmmjhhj} in Remark
\ref{ugyugg}, equality \er{jhgghfhj} implies that under the change
of coordinate system given by the Galilean Transformation
\begin{equation}\label{noninchgravortbstrjgghguittu1intmmkuk}
\begin{cases}
\vec x'=\vec x+\vec wt,\\
t'=t,
\end{cases}
\end{equation}
the quantity $\psi_1$ transforms as:
\begin{equation}\label{jhgghfhjhjhjhyuyuy}
\psi'_1:=e^{\sum_{j=1}^{n}\frac{im_j}{\hbar}(\vec w\cdot\vec
x_j+\frac{1}{2}|\vec w|^2t)}\psi_1,
\end{equation}
provided that $\psi'=\psi$. Moreover, \er{jhgghfhjhjhjhyuyuy}
coincides with the classical law of transformation of the wave
function, under the Galilean Transformation (see section 17 in
\cite{LL}, the end of the section).
\end{remark}

Next, again consider the motion and interaction of system of $n$
quantum micro-particles having inertial masses $m_1,\ldots, m_n$ and
the charges $\sigma_1,\ldots,\sigma_n$ with the known gravitational
and electromagnetical field with potentials $\vec v(\vec x,t)$,
$\vec A(\vec x,t)$ and $\Psi(\vec x,t)$ and additional conservative
field with potential $V(\vec y_1,\ldots,\vec y_n,t)$. Then consider
a Lagrangian density $L$ defined by
\begin{multline}\label{vhfffngghkjgghDDmmkkksmfggffgZZ}
L_2\left(\psi,\vec A,\Psi,\vec v,\vec x_1,\ldots,\vec x_n,t\right):=\\
\frac{i\hbar}{2}\left(\left(\frac{\partial\psi}{\partial
t}+\sum\limits_{k=1}^{n}\vec v(\vec x_k,t)\cdot\nabla_{\vec
x_k}\psi\right)\cdot\bar\psi-\psi\cdot\left(\frac{\partial\bar\psi}{\partial
t}+\sum\limits_{k=1}^{n}\vec v(\vec x_k,t)\cdot\nabla_{\vec
x_k}\bar\psi\right)\right)-\sum\limits_{k=1}^{n}\frac{\hbar^2}{2m_k}\nabla_{\vec
x_k}\psi\cdot\nabla_{\vec x_k}\bar\psi\\+V\left(\vec x_1,\ldots,\vec
x_n,t\right)\psi\cdot\bar\psi
-\sum\limits_{k=1}^{n}\frac{\hbar\sigma_k
i}{2m_kc}\left(\nabla_{\vec
x_k}\psi\cdot\bar\psi-\psi\cdot\nabla_{\vec
x_k}\bar\psi\right)\cdot\vec A(\vec
x_k,t)-\sum\limits_{k=1}^{n}\frac{\sigma^2_k}{2m_kc^2}\left|\vec
A(\vec x_k,t)\right|^2\psi\cdot\bar\psi\\
-\sum\limits_{k=1}^{n}\sigma_k\left(\Psi(\vec x_k,t)-\frac{1}{c}\vec
v(\vec x_k,t)\cdot\vec A(\vec x_k,t)\right)\psi\cdot\bar\psi,
\end{multline}
where $\psi(\vec x_1,\ldots,\vec x_n,t)\in\mathbb{C}$ is a wave
function of the system. Then, as before, it can be proven that $L$
is invariant under the change of inertial or non-inertial cartesian
coordinate systems of the form
\begin{equation*}
\begin{cases}
t'=t\\
\vec x'_k=A(t)\cdot\vec x_k+\vec
z(t)\quad\quad\forall k=1,\ldots,n,\end{cases}
\end{equation*}
provided that $\psi'=\psi$. We investigate stationary points of the
functional
\begin{equation}\label{btfffygtgyggyDDmmkkksmZZ}
J=\int_0^T\int_{\left(\mathbb{R}^3\right)^n}L\left(\psi,\vec
A,\Psi,\vec v,\vec x_1,\ldots,\vec x_n,t\right)d\vec x_1\ldots,d\vec
x_n dt.
\end{equation}
Then,
\begin{multline}\label{vhfffngghkjgghDDmmkkkghgsmZZ}
0=\frac{\delta L_2}{\delta (\bar\psi)}=
i\hbar\left(\frac{\partial\psi}{\partial
t}+\sum\limits_{k=1}^{n}\frac{1}{2}\vec v(\vec
x_k,t)\cdot\nabla_{\vec
x_k}\psi+\sum\limits_{k=1}^{n}\frac{1}{2}div_{\vec
x_k}\left\{\psi\vec v(\vec
x_k,t)\right\}\right)+\sum\limits_{k=1}^{n}\frac{\hbar^2}{2m_k}\Delta_{\vec
x_k}\psi\\+V\left(\vec x_1,\ldots,\vec x_n,t\right)\psi
-\sum\limits_{k=1}^{n}\frac{\hbar\sigma_k i}{2m_kc}\left(\vec A(\vec
x_k,t)\cdot\nabla_{\vec x_k}\psi+div_{\vec x_k}\left\{\psi\vec
A(\vec
x_k,t)\right\}\right)-\sum\limits_{k=1}^{n}\frac{\sigma^2_k}{2m_kc^2}\left|\vec
A(\vec x_k,t)\right|^2\psi
\\-\sum\limits_{k=1}^{n}\sigma_k\left(\Psi(\vec x_k,t)-\frac{1}{c}\vec
v(\vec x_k,t)\cdot\vec A(\vec x_k,t)\right)\psi,
\end{multline}
and
\begin{multline}\label{vhfffngghkjgghDDmmkkkghguhyuysmZZ}
0=\frac{\delta L_2}{\delta (\psi)}= \bar
{(i)}\hbar\left(\frac{\partial\bar\psi}{\partial
t}+\sum\limits_{k=1}^{n}\frac{1}{2}\vec v(\vec
x_k,t)\cdot\nabla_{\vec
x_k}\bar\psi+\sum\limits_{k=1}^{n}\frac{1}{2}div_{\vec
x_k}\left\{\bar\psi\vec v(\vec
x_k,t)\right\}\right)+\sum\limits_{k=1}^{n}\frac{\hbar^2}{2m_k}\Delta_{\vec
x_k}\bar\psi \\+V\left(\vec x_1,\ldots,\vec x_n,t\right)\bar\psi
-\sum\limits_{k=1}^{n}\frac{\hbar\sigma_k \bar
{(i)}}{2m_kc}\left(\vec A(\vec x_k,t)\cdot\nabla_{\vec
x_k}\bar\psi+div_{\vec x_k}\left\{\bar\psi\vec A(\vec
x_k,t)\right\}\right)-\sum\limits_{k=1}^{n}\frac{\sigma^2_k}{2m_kc^2}\left|\vec
A(\vec x_k,t)\right|^2\bar\psi
\\-\sum\limits_{k=1}^{n}\sigma_k\left(\Psi(\vec x_k,t)-\frac{1}{c}\vec v(\vec x_k,t)\cdot\vec
A(\vec x_k,t)\right)\bar\psi,
\end{multline}
where the last equality is just the complex conjugate of
\er{vhfffngghkjgghDDmmkkkghgsmZZ}. So we get that the Euler-Lagrange
equation for \er{btfffygtgyggyDDmmkkksmZZ} coincides with the
Schr\"{o}dinger equation of the form
\er{vhfffngghkjgghfjjghghghhjghjgghkghgghjhggjjkgfgdiyfgfjkjgjgSYSPN}.

Finally, assume that the first and the second particles have the
same mass $m_1=m_2$ and the same charge $\sigma_1=\sigma_2$ and
moreover, assume that we have
$$V\left(\vec x_1,\vec x_2,\vec x_3,\ldots,\vec
x_n,t\right)=V\left(\vec x_2,\vec x_1,\vec x_3,\ldots,\vec
x_n,t\right)\quad\quad\forall\,\vec x_1,\vec x_2,\vec
x_3,\ldots,\vec x_n\in\mathbb{R}^3,\quad\forall t.$$ In this case it
can be easily deduced that if  $\psi\left(\vec x_1,\vec x_2,\vec
x_3,\ldots,\vec x_n,t\right)$ is a solution of
\er{vhfffngghkjgghfjjghghghhjghjgghkghgghjhggjjkgfgdiyfgfjkjgjgSYSPN},
then $\psi(\vec x_2,\vec x_1,\vec x_3,\ldots,\vec x_n,t)$ as also a
solutions of
\er{vhfffngghkjgghfjjghghghhjghjgghkghgghjhggjjkgfgdiyfgfjkjgjgSYSPN},
Therefore, if $\psi\left(\vec x_1,\vec x_2,\vec x_3,\ldots,\vec
x_n,t\right)$ is a solution of
\er{vhfffngghkjgghfjjghghghhjghjgghkghgghjhggjjkgfgdiyfgfjkjgjgSYSPN},
then for every $t\geq 0$ we will have either
\begin{equation}\label{fyfgjguyiuiHHgjgjgghbmgjgfhfhhHHrr11gg}
\psi(\vec x_2,\vec x_1,\vec x_3,\ldots,\vec x_n,t)=-\psi\left(\vec
x_1,\vec x_2,\vec x_3,\ldots,\vec
x_n,t\right)\quad\quad\forall\,\vec x_1,\vec x_2,\vec
x_3,\ldots,\vec x_n\in\mathbb{R}^3,
\end{equation}
or
\begin{equation}\label{fyfgjguyiuiHHgjgjgghbmgjgfhfhhHHrr22gg}
\psi(\vec x_2,\vec x_1,\vec x_3,\ldots,\vec x_n,t)=\psi\left(\vec
x_1,\vec x_2,\vec x_3,\ldots,\vec
x_n,t\right)\quad\quad\forall\,\vec x_1,\vec x_2,\vec
x_3,\ldots,\vec x_n\in\mathbb{R}^3,
\end{equation}
provided that either \er{fyfgjguyiuiHHgjgjgghbmgjgfhfhhHHrr11gg} or
\er{fyfgjguyiuiHHgjgjgghbmgjgfhfhhHHrr22gg} respectively holds for
the initial instant of time $t=0$. So we have a consistency with the
principles of identity for two or more identical fermions or bosons,
in the cases where we do not take into account the spin interaction.

\subsection{Quantum Liouville's equation for a finite system of quantum
particles}\label{hilghjhghfdgdg} Consider the
\underline{statistical} description of the motion of a system of $n$
quantum micro-particles with inertial masses $m_1,\ldots,m_n$ and
the charges $\sigma_1,\ldots,\sigma_n$ in the outer gravitational
and electromagnetical field with characteristics $\vec v(\vec x,t)$,
$\vec A(\vec x,t)$ and $\Psi(\vec x,t)$ and additional conservative
field with potential $V(\vec y_1,\ldots,\vec y_n,t)$, not taking
into account the spin interaction. Then it is well known that the
Quantum Liouville's equation for this system of particles has the
following form:
\begin{multline}\label{vhfffngghkjgghfjjghghghhjghjgghkghggkghghjghSYSPNjj}
i\hbar\frac{\partial\xi}{\partial t}\left(\vec x_1,\ldots,\vec
x_n,\vec y_1,\ldots,\vec y_n,t\right)=\\ \hat H_0\left(\vec
x_1,\ldots,\vec x_n,t\right)\cdot\xi\left(\vec x_1,\ldots,\vec
x_n,\vec y_1,\ldots,\vec y_n,t\right)-\hat H^*_0\left(\vec
y_1,\ldots,\vec y_n,t\right)\cdot\xi\left(\vec x_1,\ldots,\vec
x_n,\vec y_1,\ldots,\vec y_n,t\right),
\end{multline}
where $\xi\left(\vec x_1,\ldots,\vec x_n,\vec y_1,\ldots,\vec
y_n,t\right)\in\mathbb{C}$ is a density-matrix function and $\hat
H_0\left(\vec x_1,\ldots,\vec x_n,t\right)$ is the Hamiltonian
operator acting on the variables $\left(\vec x_1,\ldots,\vec
x_n\right)$ and $\hat H^*_0\left(\vec y_1,\ldots,\vec y_n,t\right)$
is the \underline{complex} conjugate (not the Hermitian adjoint)  to
the Hamiltonian operator acting on the variables $\left(\vec
y_1,\ldots,\vec y_n\right)$. Since by
\er{vhfffngghkjgghfjjghghghSYSPN} the Hamiltonian for a
macro-particles has the form
\begin{multline}\label{vhfffngghkjgghfjjghghghhjghjgghkghggSYSPNjj}
H_{\text{macro}}\left(\vec P_1,\ldots,\vec P_n,\vec r_1,\ldots,\vec
r_n,t\right)=-V\left(\vec r_1,\ldots,\vec r_n,t\right)+
\sum_{j=1}^{n}\frac{1}{2}\vec P_j\cdot\vec v(\vec
r_j,t)+\sum_{j=1}^{n}\frac{1}{2}\vec v(\vec r_j,t)\cdot\vec P_j\\+
\sum_{j=1}^{n}\frac{1}{2m_j}\left|\vec P_j-\frac{\sigma_j}{c}\vec
A(\vec r_j,t)\right|^2+\sum_{j=1}^{n}\sigma_j\left(\Psi(\vec
r_j,t)-\frac{1}{c}\vec A(\vec r_j,t)\cdot\vec v(\vec r_j,t)\right),
\end{multline}
as before in
\er{vhfffngghkjgghfjjghghghhjghjgghkghgghjhggjjkgSYSPN}, we built
the Hamiltonian operator as
\begin{multline}\label{vhfffngghkjgghfjjghghghhjghjgghkghgghjhggjjkgSYSPNjj}
\hat H_0\left(\vec x_1,\ldots,\vec x_n,t\right)\cdot\psi\left(\vec
x_1,\ldots,\vec x_n,t\right)=
-\sum_{j=1}^{n}\frac{\hbar^2}{2m_j}\Delta_{\vec
x_j}\psi-\sum_{j=1}^{n}\frac{i\hbar}{2}div_{\vec x_j}\left\{\psi\vec
v(\vec x_j,t)\right\}\\-\sum_{j=1}^{n}\frac{i\hbar}{2}\vec v(\vec
x_j,t)\cdot\nabla_{\vec x_j}\psi+\sum_{j=1}^{n}\frac{
i\hbar\sigma_j}{2m_jc}div_{\vec x_j}\left\{\psi\vec A(\vec
x_j,t)\right\}+\sum_{j=1}^{n}\frac{ i\hbar\sigma_j}{2m_jc}\vec
A(\vec x_j,t)\cdot\nabla_{\vec
x_j}\psi\\+\sum_{j=1}^{n}\left(\sigma_j\Psi(\vec
x_j,t)-\frac{\sigma_j}{c}\vec A(\vec x_j,t)\cdot\vec v(\vec
x_j,t)+\frac{\sigma^2_j}{2m_jc^2}\left|\vec A(\vec
x_j,t)\right|^2\right)\psi-V\left(\vec x_1,\ldots,\vec
x_n,t\right)\psi,
\end{multline}
and consistently with
\er{vhfffngghkjgghfjjghghghhjghjgghkghgghjhggjjkgSYSPNjj}:
\begin{multline}\label{vhfffngghkjgghfjjghghghhjghjgghkghgghjhggjjkgSYSPNjjklkl}
\hat H^*_0\left(\vec y_1,\ldots,\vec y_n,t\right)\cdot\psi\left(\vec
y_1,\ldots,\vec y_n,t\right)=
-\sum_{j=1}^{n}\frac{\hbar^2}{2m_j}\Delta_{\vec
y_j}\psi+\sum_{j=1}^{n}\frac{i\hbar}{2}div_{\vec y_j}\left\{\psi\vec
v(\vec y_j,t)\right\}\\+\sum_{j=1}^{n}\frac{i\hbar}{2}\vec v(\vec
y_j,t)\cdot\nabla_{\vec y_j}\psi-\sum_{j=1}^{n}\frac{
i\hbar\sigma_j}{2m_jc}div_{\vec y_j}\left\{\psi\vec A(\vec
y_j,t)\right\}-\sum_{j=1}^{n}\frac{ i\hbar\sigma_j}{2m_jc}\vec
A(\vec y_j,t)\cdot\nabla_{\vec
y_j}\psi\\+\sum_{j=1}^{n}\left(\sigma_j\Psi(\vec
y_j,t)-\frac{\sigma_j}{c}\vec A(\vec y_j,t)\cdot\vec v(\vec
y_j,t)+\frac{\sigma^2_j}{2m_jc^2}\left|\vec A(\vec
y_j,t)\right|^2\right)\psi-V\left(\vec y_1,\ldots,\vec
y_n,t\right)\psi.
\end{multline}
Thus we rewrite the corresponding Quantum Liouville's equation
\er{vhfffngghkjgghfjjghghghhjghjgghkghggkghghjghSYSPNjj} as:
%
%
%
%
%
%
\begin{multline}\label{vhfffngghkjgghfjjghghghhjghjgghkghgghjhggjjkgfgdiyfgfjkjgjgSYSPNjj}
i\hbar\left(\frac{\partial\xi}{\partial t}+\sum_{j=1}^{n}\vec v(\vec
x_j,t)\cdot\nabla_{\vec x_j}\xi+\sum_{j=1}^{n}\vec v(\vec
y_j,t)\cdot\nabla_{\vec
y_j}\xi\right)+\sum_{j=1}^{n}\frac{i\hbar}{2}\left(div_{\vec
x_j}\vec v(\vec x_j,t)+div_{\vec y_j}\vec v(\vec y_j,t)\right)\xi=
\\
-\sum_{j=1}^{n}\frac{\hbar^2}{2m_j}\left(\Delta_{\vec
x_j}\xi-\Delta_{\vec y_j}\xi\right)-\left(V\left(\vec
x_1,\ldots,\vec x_n,t\right)-V\left(\vec y_1,\ldots,\vec
y_n,t\right)\right)\xi\\+\sum_{j=1}^{n}\frac{
i\hbar\sigma_j}{2m_jc}\left(div_{\vec x_j}\left\{\xi\vec A(\vec
x_j,t)\right\}+div_{\vec y_j}\left\{\xi\vec A(\vec
y_j,t)\right\}\right)+\sum_{j=1}^{n}\frac{
i\hbar\sigma_j}{2m_jc}\left(\vec A(\vec x_j,t)\cdot\nabla_{\vec
x_j}\xi+\vec A(\vec y_j,t)\cdot\nabla_{\vec
y_j}\xi\right)\\+\sum_{j=1}^{n}\left(\sigma_j\Psi(\vec
x_j,t)-\frac{\sigma_j}{c}\vec A(\vec x_j,t)\cdot\vec v(\vec
x_j,t)+\frac{\sigma^2_j}{2m_jc^2}\left|\vec A(\vec
x_j,t)\right|^2\right)\xi\\- \sum_{j=1}^{n}\left(\sigma_j\Psi(\vec
y_j,t)-\frac{\sigma_j}{c}\vec A(\vec y_j,t)\cdot\vec v(\vec
y_j,t)+\frac{\sigma^2_j}{2m_jc^2}\left|\vec A(\vec
y_j,t)\right|^2\right)\xi.
\end{multline}
Next consider a change of some non-inertial cartesian coordinate
system $(*)$ to another cartesian coordinate system $(**)$ of the
form \er{noninchgravortbstrjgghguittu2}:
\begin{equation}\label{noninchgravortbstrghgggSYSPNjj}
\begin{cases}
\vec x'=A(t)\cdot\vec x+\vec z(t),\\
t'=t,
\end{cases}
\end{equation}
where $A(t)\in SO(3)$ is a rotation. Then, since
\begin{equation}\label{vyfgjhgjhvhgghSYSPNjj}
\begin{cases}
\vec x'_j=A(t)\cdot\vec x_j+\vec z(t)\quad\forall j=1,\ldots,n\\
\vec y'_j=A(t)\cdot\vec y_j+\vec z(t)\quad\forall j=1,\ldots,n\\
\xi'\left(\vec x'_1,\ldots,\vec x'_n,\vec y'_1,\ldots,\vec
y'_n,t\right)=\xi\left(\vec x_1,\ldots,\vec x_n,\vec y_1,\ldots,\vec
y_n,t\right)
\\
V'\left(\vec y'_1,\ldots,\vec y'_n,t\right)=V\left(\vec
y_1,\ldots,\vec y_n,t\right)
\\
\vec A'=A(t)\cdot\vec A
\\
\Psi'-\frac{1}{c}\vec A'\cdot\vec v'=\Psi-\frac{1}{c}\vec A\cdot\vec
v,
\end{cases}
\end{equation}
we deduce that  the Quantum Liouville's equation  equation of the
form
\er{vhfffngghkjgghfjjghghghhjghjgghkghgghjhggjjkgfgdiyfgfjkjgjgSYSPNjj}
is invariant under the change of non-inertial cartesian coordinate
system, provided we have $\xi'=\xi$.

Next assume that $\hat A\left(\vec x_1,\ldots,\vec x_n,t\right)$ is
some Hermitian operator acting on the functions\\ $\psi\left(\vec
x_1,\ldots,\vec x_n,t\right)\in \mathbb{C}$ with respect to spatial
variables $\left(\vec x_1,\ldots,\vec x_n\right)$. Then the average
$\tilde A(t)$ of $\hat A$ on the density matrix $\xi$ is defined by:
\begin{equation}\label{gyfyfgfgfgghZZHHhuggkjhjjgg22}
\tilde A(t)=\frac{\int \left(\hat A\left(\vec x_1,\ldots,\vec
x_n,t\right)\cdot\xi\right)\left(\vec z_1,\ldots,\vec z_n,\vec
z_1,\ldots,\vec z_n,t\right)d\vec z_1\ldots d\vec z_n}{\int
\xi\left(\vec z_1,\ldots,\vec z_n,\vec z_1,\ldots,\vec
z_n,t\right)d\vec z_1\ldots d\vec z_n}.
\end{equation}
Thus, using the fact that $\hat A$ is Hermitian, it can be easily
deduced that in addition to \er{gyfyfgfgfgghZZHHhuggkjhjjgg22} we
have the following identity:
\begin{multline}\label{gyfyfgfgfgghZZHHhuggkjhjjggjhhh22}
\int \left(\hat A\left(\vec x_1,\ldots,\vec
x_n,t\right)\cdot\xi\right)\left(\vec z_1,\ldots,\vec z_n,\vec
z_1,\ldots,\vec z_n,t\right)d\vec z_1\ldots d\vec z_n=\\ \int
\left(\hat A^*\left(\vec y_1,\ldots,\vec
y_n,t\right)\cdot\xi\right)\left(\vec z_1,\ldots,\vec z_n,\vec
z_1,\ldots,\vec z_n,t\right)d\vec z_1\ldots d\vec z_n,
\end{multline}
where $\hat A^*\left(\vec y_1,\ldots,\vec y_n,t\right)$ is the
\underline{complex} conjugate (not the Hermitian adjoint)  to the
operator $\hat A$ acting on the variables $\left(\vec
y_1,\ldots,\vec y_n\right)$ and so,
\begin{equation}\label{gyfyfgfgfgghZZHHhuggkjhjjggjhhh22jkhjjkkjhjihjhj}
\tilde A(t)=\frac{\int \left(\hat A^*\left(\vec y_1,\ldots,\vec
y_n,t\right)\cdot\xi\right)\left(\vec z_1,\ldots,\vec z_n,\vec
z_1,\ldots,\vec z_n,t\right)d\vec z_1\ldots d\vec z_n}{\int
\xi\left(\vec z_1,\ldots,\vec z_n,\vec z_1,\ldots,\vec
z_n,t\right)d\vec z_1\ldots d\vec z_n}.
\end{equation}
In particular, by inserting \er{gyfyfgfgfgghZZHHhuggkjhjjggjhhh22}
in the particular case $\hat A=\hat H_0$ into
\er{vhfffngghkjgghfjjghghghhjghjgghkghggkghghjghSYSPNjj} we deduce:
\begin{equation}\label{gyfyfgfgfgghZZHHhuggkjhjjggjhhhguggh22}
i\hbar\frac{\partial}{\partial t}\int \xi\left(\vec z_1,\ldots,\vec
z_n,\vec z_1,\ldots,\vec z_n,t\right)d\vec z_1\ldots d\vec z_n=0,
\end{equation}
and so,
\begin{equation}\label{gyfyfgfgfgghZZHHhuggkjhjjggjhhhjhgg22}
\int \xi\left(\vec z_1,\ldots,\vec z_n,\vec z_1,\ldots,\vec
z_n,t\right)d\vec z_1\ldots d\vec z_n=\int \xi\left(\vec
z_1,\ldots,\vec z_n,\vec z_1,\ldots,\vec z_n,0\right)d\vec z_1\ldots
d\vec z_n.
\end{equation}

Next, we remind that by
\er{vhfffngghkjgghfjjghghghhjghjgghkghggkghghjghSYSPNjj} we have
\begin{multline}\label{vhfffngghkjgghfjjghghghhjghjgghkghggkghghjghSYSPNjjjhhh}
i\hbar\frac{\partial\xi}{\partial t}\left(\vec x_1,\ldots,\vec
x_n,\vec y_1,\ldots,\vec y_n,t\right)=\\ \hat H_0\left(\vec
x_1,\ldots,\vec x_n,t\right)\cdot\xi\left(\vec x_1,\ldots,\vec
x_n,\vec y_1,\ldots,\vec y_n,t\right)-\hat H^*_0\left(\vec
y_1,\ldots,\vec y_n,t\right)\cdot\xi\left(\vec x_1,\ldots,\vec
x_n,\vec y_1,\ldots,\vec y_n,t\right).
\end{multline}
Thus if we denote,
\begin{equation}\label{gyfyfgfgfgghZZHHhuggkjhjjggjhhhjhgg22jjjkjk}
\xi_1\left(\vec x_1,\ldots,\vec x_n,\vec y_1,\ldots,\vec
y_n,t\right):=\bar\xi\left(\vec y_1,\ldots,\vec y_n,\vec
x_1,\ldots,\vec x_n,t\right),
\end{equation}
and
\begin{equation}\label{gyfyfgfgfgghZZHHhuggkjhjjggjhhhjhgg22jjjkjkkljjjkkj}
\zeta\left(\vec x_1,\ldots,\vec x_n,\vec y_1,\ldots,\vec
y_n,t\right):=\xi\left(\vec x_1,\ldots,\vec x_n,\vec y_1,\ldots,\vec
y_n,t\right)-\bar\xi\left(\vec y_1,\ldots,\vec y_n,\vec
x_1,\ldots,\vec x_n,t\right),
\end{equation}
then
\begin{multline}\label{vhfffngghkjgghfjjghghghhjghjgghkghggkghghjghSYSPNjjjhhhggg}
-i\hbar\frac{\partial\bar\xi}{\partial t}\left(\vec y_1,\ldots,\vec
y_n,\vec x_1,\ldots,\vec x_n,t\right)=\\ \hat H^*_0\left(\vec
y_1,\ldots,\vec y_n,t\right)\cdot\bar\xi\left(\vec y_1,\ldots,\vec
y_n,\vec x_1,\ldots,\vec x_n,t\right)-\hat H_0\left(\vec
x_1,\ldots,\vec x_n,t\right)\cdot\bar\xi\left(\vec y_1,\ldots,\vec
y_n,\vec x_1,\ldots,\vec x_n,t\right),
\end{multline}
and then
implies\begin{multline}\label{vhfffngghkjgghfjjghghghhjghjgghkghggkghghjghSYSPNjjjhhhhjhjjhjkhjhh2}
i\hbar\frac{\partial\xi_1}{\partial t}\left(\vec x_1,\ldots,\vec
x_n,\vec y_1,\ldots,\vec y_n,t\right)=\\ \hat H_0\left(\vec
x_1,\ldots,\vec x_n,t\right)\cdot\xi_1\left(\vec x_1,\ldots,\vec
x_n,\vec y_1,\ldots,\vec y_n,t\right)-\hat H^*_0\left(\vec
y_1,\ldots,\vec y_n,t\right)\cdot\xi_1\left(\vec x_1,\ldots,\vec
x_n,\vec y_1,\ldots,\vec y_n,t\right).
\end{multline}
Therefore,
\begin{multline}\label{vhfffngghkjgghfjjghghghhjghjgghkghggkghghjghSYSPNjjjhhhhjhjjhjkhjhh}
i\hbar\frac{\partial\zeta}{\partial t}\left(\vec x_1,\ldots,\vec
x_n,\vec y_1,\ldots,\vec y_n,t\right)=\\ \hat H_0\left(\vec
x_1,\ldots,\vec x_n,t\right)\cdot\zeta\left(\vec x_1,\ldots,\vec
x_n,\vec y_1,\ldots,\vec y_n,t\right)-\hat H^*_0\left(\vec
y_1,\ldots,\vec y_n,t\right)\cdot\zeta\left(\vec x_1,\ldots,\vec
x_n,\vec y_1,\ldots,\vec y_n,t\right).
\end{multline}
Thus, if $\zeta$ satisfies initial condition $\zeta\left(\vec
x_1,\ldots,\vec x_n,\vec y_1,\ldots,\vec y_n,0\right)=0$, then
$\zeta\left(\vec x_1,\ldots,\vec x_n,\vec y_1,\ldots,\vec
y_n,t\right)=0$ for every $t>0$. So
\begin{multline}\label{gyfyfgfgfgghZZHHhuggkjhjjggjhhhjhgg22jjjkjkkljjjkkjopiiio}
\xi\left(\vec x_1,\ldots,\vec x_n,\vec y_1,\ldots,\vec
y_n,0\right)=\bar\xi\left(\vec y_1,\ldots,\vec y_n,\vec
x_1,\ldots,\vec x_n,0\right)\quad\text{implies}\quad\\ \xi\left(\vec
x_1,\ldots,\vec x_n,\vec y_1,\ldots,\vec
y_n,t\right)=\bar\xi\left(\vec y_1,\ldots,\vec y_n,\vec
x_1,\ldots,\vec x_n,t\right)\quad\forall\, t\geq 0.
\end{multline}

Finally, assume that the first and the second particles have the
same mass $m_1=m_2$ and the same charge $\sigma_1=\sigma_2$ and
moreover, assume that we have
$$V\left(\vec x_1,\vec x_2,\vec x_3,\ldots,\vec
x_n,t\right)=V\left(\vec x_2,\vec x_1,\vec x_3,\ldots,\vec
x_n,t\right)\quad\quad\forall\,\vec x_1,\vec x_2,\vec
x_3,\ldots,\vec x_n\in\mathbb{R}^3,\quad\forall t.$$ In this case it
can be easily deduced that if  $\xi\left(\vec x_1,\vec x_2,\vec
x_3,\ldots,\vec x_n,\vec y_1,\vec y_2,\vec y_3,\ldots,\vec
y_n,t\right)$ is a solution of
\er{vhfffngghkjgghfjjghghghhjghjgghkghggkghghjghSYSPNjjjhhh}, then
$\xi(\vec x_2,\vec x_1,\vec x_3,\ldots,\vec x_n,\vec y_2,\vec
y_1,\vec y_3,\ldots,\vec y_n,t)$
is also a solution of
\er{vhfffngghkjgghfjjghghghhjghjgghkghggkghghjghSYSPNjjjhhh}.
Therefore, if $\xi\left(\vec x_1,\vec x_2,\vec x_3,\ldots,\vec
x_n,\vec y_1,\vec y_2,\vec y_3,\ldots,\vec y_n,t\right)$ is a
solution of
\er{vhfffngghkjgghfjjghghghhjghjgghkghggkghghjghSYSPNjjjhhh}, then
for every $t\geq 0$ we will have
%
%
%
%
%
%
\begin{multline}\label{fyfgjguyiuiHHgjgjgghbmgjgfhfhhHHrr22}
\xi(\vec x_2,\vec x_1,\vec x_3,\ldots,\vec x_n,\vec y_2,\vec
y_1,\vec y_3,\ldots,\vec y_n,t)=\xi\left(\vec x_1,\vec x_2,\vec
x_3,\ldots,\vec x_n,\vec y_1,\vec y_2,\vec y_3,\ldots,\vec
y_n,t\right)
\\ \quad\quad\forall\,\vec x_1,\vec x_2,\vec
x_3,\ldots,\vec x_n,\vec y_1,\vec y_2,\vec y_3,\ldots,\vec
y_n\in\mathbb{R}^3,
\end{multline}
provided that
\er{fyfgjguyiuiHHgjgjgghbmgjgfhfhhHHrr22} holds for the initial
instant of time $t=0$. So we have a consistency with the principles
of identity for two or more identical fermions or bosons, in the
cases where we do not take into account the spin interaction.

Next given $\xi\left(\vec x_1,\ldots,\vec x_n,\vec y_1,\ldots,\vec
y_n,t\right)$ that satisfies
\begin{equation}\label{gyfyfgfgfgghZZHHhuggkjhjjggjhhhjhgg22jjjkjkkljjjkkjopiiiojjkljkjk}
\xi\left(\vec x_1,\ldots,\vec x_n,\vec y_1,\ldots,\vec
y_n,t\right)=\bar\xi\left(\vec y_1,\ldots,\vec y_n,\vec
x_1,\ldots,\vec x_n,t\right)\quad\forall\, t\geq 0,
\end{equation}
define the operator $\hat R_\xi$ by:
\begin{equation}\label{gyfyfgfgfgghZZHHhuggkjhjjggjhhhjhgg22jjjkjkkljjjkkjopiiiojjkljkjkjkljkjk}
\hat R_\xi\left(\vec x_1,\ldots,\vec x_n,t\right)\cdot\psi\left(\vec
x_1,\ldots,\vec x_n,t\right)=\int\xi\left(\vec x_1,\ldots,\vec
x_n,\vec y_1,\ldots,\vec y_n,t\right)\psi\left(\vec y_1,\ldots,\vec
y_n,t\right)d\vec y_1\ldots d\vec y_n.
\end{equation}
Then the mapping $\xi\to\hat R_\xi$ is one-to-one. Moreover, by
\er{gyfyfgfgfgghZZHHhuggkjhjjggjhhhjhgg22jjjkjkkljjjkkjopiiiojjkljkjk}
$\hat R_\xi$ is an Hermitian operator. Finally by
\er{vhfffngghkjgghfjjghghghhjghjgghkghggkghghjghSYSPNjjjhhh} we have
\begin{multline}\label{vhfffngghkjgghfjjghghghhjghjgghkghggkghghjghSYSPNjjjhhhjjkjkjklkjljl}
i\hbar\frac{\partial\hat R_\xi}{\partial t}\left(\vec
x_1,\ldots,\vec x_n,t\right)=\\ \hat H_0\left(\vec x_1,\ldots,\vec
x_n,t\right)\cdot\hat R_\xi\left(\vec x_1,\ldots,\vec
x_n,t\right)-\hat R_\xi\left(\vec x_1,\ldots,\vec
x_n,t\right)\cdot\hat H_0\left(\vec x_1,\ldots,\vec x_n,t\right).
\end{multline}
Equation
\er{vhfffngghkjgghfjjghghghhjghjgghkghggkghghjghSYSPNjjjhhhjjkjkjklkjljl}
is equivalent to
\er{vhfffngghkjgghfjjghghghhjghjgghkghggkghghjghSYSPNjjjhhh}.

Next, clearly $\xi$ satisfies
\er{fyfgjguyiuiHHgjgjgghbmgjgfhfhhHHrr22} if and only if we have
either the following:
\begin{multline}\label{fyfgjguyiuiHHgjgjgghbmgjgfhfhhHHljjjhjhjhyyuyuihhjjk}
\psi(\vec x_2,\vec x_1,\vec x_3,\ldots,\vec x_n,t)=-\psi(\vec
x_1,\vec x_2,\vec x_3,\ldots,\vec x_n,t)\;\;\forall\,\vec
x_1,\ldots,\vec x_n\in\mathbb{R}^3\quad\text{and}\quad\phi=\hat
R_\xi\cdot\psi\\ \quad\text{implies}\quad \phi(\vec x_2,\vec
x_1,\vec x_3,\ldots,\vec x_n,t)=-\phi(\vec x_1,\vec x_2,\vec
x_3,\ldots,\vec x_n,t)\;\;\forall\,\vec x_1,\ldots,\vec
x_n\in\mathbb{R}^3,
\end{multline}
or the following:
\begin{multline}\label{fyfgjguyiuiHHgjgjgghbmgjgfhfhhHHljjjhjhjhyyuyuihhjhj}
\psi(\vec x_2,\vec x_1,\vec x_3,\ldots,\vec x_n,t)=\psi(\vec
x_1,\vec x_2,\vec x_3,\ldots,\vec x_n,t)\;\;\forall\,\vec
x_1,\ldots,\vec x_n\in\mathbb{R}^3\quad\text{and}\quad\phi=\hat
R_\xi\cdot\psi\\ \quad\text{implies}\quad \phi(\vec x_2,\vec
x_1,\vec x_3,\ldots,\vec x_n,t)=\phi(\vec x_1,\vec x_2,\vec
x_3,\ldots,\vec x_n,t)\;\;\forall\,\vec x_1,\ldots,\vec
x_n\in\mathbb{R}^3.
\end{multline}

Finally, assume that $\hat A\left(\vec x_1,\ldots,\vec x_n,t\right)$
is some Hermitian operator acting on the functions\\ $\psi\left(\vec
x_1,\ldots,\vec x_n,t\right)\in \mathbb{C}$ with respect to spatial
variables $\left(\vec x_1,\ldots,\vec x_n\right)$. Then we remind
that the average $\tilde A(t)$ of $\hat A$ on the density matrix
$\xi$, is defined by \er{gyfyfgfgfgghZZHHhuggkjhjjgg22}:
\begin{equation}\label{gyfyfgfgfgghZZHHhuggkjhjjgg22zz}
\tilde A(t)=\frac{\int \left(\hat A\left(\vec x_1,\ldots,\vec
x_n,t\right)\cdot\xi\right)\left(\vec z_1,\ldots,\vec z_n,\vec
z_1,\ldots,\vec z_n,t\right)d\vec z_1\ldots d\vec z_n}{\int
\xi\left(\vec z_1,\ldots,\vec z_n,\vec z_1,\ldots,\vec
z_n,t\right)d\vec z_1\ldots d\vec z_n}.
\end{equation}
Thus,
\begin{equation}\label{gyfyfgfgfgghZZHHhuggkjhjjgg22hjhjjhjh}
trace \left(\hat A\circ \hat R_\xi\right)=\int \left(\hat
A\left(\vec x_1,\ldots,\vec x_n,t\right)\cdot\xi\right)\left(\vec
z_1,\ldots,\vec z_n,\vec z_1,\ldots,\vec z_n,t\right)d\vec z_1\ldots
d\vec z_n,
\end{equation}
and so,
\begin{equation}\label{gyfyfgfgfgghZZHHhuggkjhjjgg22hjhjjhjhmnbmbkjkjkh}
\tilde A(t)=\frac{trace \left(\hat A\circ \hat R_\xi\right)}{trace
\left(\hat R_\xi\right)}=\frac{trace \left( \hat R_\xi\circ \hat
A\right)}{trace \left(\hat R_\xi\right)},
\end{equation}
where by the trace we mean the trace of an operator on a Hilbert
space. Moreover, by \er{gyfyfgfgfgghZZHHhuggkjhjjggjhhhjhgg22} and
\er{gyfyfgfgfgghZZHHhuggkjhjjgg22hjhjjhjh} we have
\begin{equation}\label{gyfyfgfgfgghZZHHhuggkjhjjgg22hjhjjhjhmnbmbjkjhjhjhhhjhj}
trace \left( \hat R_\xi \right)(t)=trace \left( \hat
R_\xi\right)(0).
\end{equation}

We also prove the following:
\begin{proposition}\label{hugyfyfffgh}
Let $\hat A:=\hat A(t)$ and $\hat H_0:=\hat H_0(t)$ be two
time-dependent Hermitian operators on a given abstract Hilbert
space, such that
\begin{equation}\label{vhfffngghkjgghfjjghghghhjghjgghkghggkghghjghSYSPNjjjhhhjjkjkjklkjljljkkhhkjiijji55}
i\hbar\frac{\partial\hat A}{\partial t}(t)= \hat H_0(t)\cdot\hat
A(t)-\hat A(t)\cdot\hat H_0(t)\quad\quad\forall t.
\end{equation}
If, given the holomorphic function $f(z):\mathbb{C}\to \mathbb{C}$
defined as a sum of the power series
\begin{equation}\label{vhfffngghkjgghfjjghghghhjghjgghkghggkghghjghSYSPNjjjhhhjjkjkjklkjljljkkhhkjiijjihhjjhjhjhjhkkhhk55nn}
f(z):=\sum_{m=0}^{+\infty}a_mz^m,
\end{equation}
with $a_m\in\mathbb{C}$, we define the operator $(f\circ \hat
A):=(f\circ \hat A)(t)$ as:
\begin{equation}\label{vhfffngghkjgghfjjghghghhjghjgghkghggkghghjghSYSPNjjjhhhjjkjkjklkjljljkkhhkjiijjihhjjhjhjhjhkkhhknnhkhh55}
(f\circ \hat A)(t):=\sum_{m=0}^{+\infty}a_m\hat A^m(t),
\end{equation}
then:
\begin{equation}\label{vhfffngghkjgghfjjghghghhjghjgghkghggkghghjghSYSPNjjjhhhjjkjkjklkjljljkkhhkjiijjijhbmbgghgh55}
i\hbar\frac{\partial}{\partial t}\left(f\circ\hat  A\right)(t)= \hat
H_0(t)\cdot\left(f\circ \hat A\right)(t)-\left(f\circ \hat
A\right)(t)\cdot\hat H_0(t).
\end{equation}
\end{proposition}
\begin{proof}
We will prove now that for every $m=0,1,2,3,\ldots$ we have
\begin{equation}\label{vhfffngghkjgghfjjghghghhjghjgghkghggkghghjghSYSPNjjjhhhjjkjkjklkjljljkkhhkjiijjihhjjhjhjhjhkuou55}
i\hbar\frac{\partial\hat A^m}{\partial t}= \hat H_0\cdot\hat
A^m-\hat A^m\cdot\hat H_0.
\end{equation}
Indeed, for $m=0$ we have
\begin{equation}\label{vhfffngghkjgghfjjghghghhjghjgghkghggkghghjghSYSPNjjjhhhjjkjkjklkjljljkkhhkjiijjihhjjhjhjhjhkuouhhhhh55}
i\hbar\frac{\partial\hat A^0}{\partial t}=i\hbar\frac{\partial
{(\hat Id)}}{\partial t}=0= \hat H_0\cdot\hat{Id}-\hat{Id}\cdot\hat
H_0=\hat H_0\cdot\hat A^0-\hat A^0\cdot\hat H_0.
\end{equation}
Next, if for some $m=k$ we have
\begin{equation}\label{vhfffngghkjgghfjjghghghhjghjgghkghggkghghjghSYSPNjjjhhhjjkjkjklkjljljkkhhkjiijjihhjjhjhjhjh55}
i\hbar\frac{\partial\hat A^k}{\partial t}= \hat H_0\cdot\hat
A^k-\hat A^k\cdot\hat H_0,
\end{equation}
then by
\er{vhfffngghkjgghfjjghghghhjghjgghkghggkghghjghSYSPNjjjhhhjjkjkjklkjljljkkhhkjiijji55}
and
\er{vhfffngghkjgghfjjghghghhjghjgghkghggkghghjghSYSPNjjjhhhjjkjkjklkjljljkkhhkjiijjihhjjhjhjhjh55}
we obtain
\begin{multline}\label{vhfffngghkjgghfjjghghghhjghjgghkghggkghghjghSYSPNjjjhhhjjkjkjklkjljljkkhhkjiijjihhjjhjhjhjhhihh55}
i\hbar\frac{\partial\hat A^{k+1}}{\partial t}=\hat
A\cdot\left(i\hbar\frac{\partial\hat A^k}{\partial
t}\right)+\left(i\hbar\frac{\partial\hat A}{\partial
t}\right)\cdot\hat A^k\\ =\hat A\cdot\left(\hat H_0\cdot\hat
A^k-\hat A^k\cdot\hat H_0\right)+ \left(\hat H_0\cdot\hat A-\hat
A\cdot\hat H_0\right)\cdot\hat A^k =\hat H_0\cdot\hat A^{k+1}-\hat
A^{k+1}\cdot\hat H_0.
\end{multline}
So we prove
\er{vhfffngghkjgghfjjghghghhjghjgghkghggkghghjghSYSPNjjjhhhjjkjkjklkjljljkkhhkjiijjihhjjhjhjhjhkuou55}
for $m=k+1$ and thus for every $m$.

Therefore, given the holomorphic function $f(z):\mathbb{C}\to
\mathbb{C}$ defined as a sum of the power series
\begin{equation}\label{vhfffngghkjgghfjjghghghhjghjgghkghggkghghjghSYSPNjjjhhhjjkjkjklkjljljkkhhkjiijjihhjjhjhjhjhkkhhk55rr}
f(z):=\sum_{m=0}^{+\infty}a_mz^m,
\end{equation}
with $a_m\in\mathbb{C}$, if we define the operator $f\circ \hat A$
as:
\begin{equation}\label{vhfffngghkjgghfjjghghghhjghjgghkghggkghghjghSYSPNjjjhhhjjkjkjklkjljljkkhhkjiijjihhjjhjhjhjhkkhhknnhkhh55xx}
f\circ \hat A:=\sum_{m=0}^{+\infty}a_m\hat A^m,
\end{equation}
by
\er{vhfffngghkjgghfjjghghghhjghjgghkghggkghghjghSYSPNjjjhhhjjkjkjklkjljljkkhhkjiijjihhjjhjhjhjhkuou55}
and
\er{vhfffngghkjgghfjjghghghhjghjgghkghggkghghjghSYSPNjjjhhhjjkjkjklkjljljkkhhkjiijjihhjjhjhjhjhkkhhknnhkhh55xx}
we finally deduce
\er{vhfffngghkjgghfjjghghghhjghjgghkghggkghghjghSYSPNjjjhhhjjkjkjklkjljljkkhhkjiijjijhbmbgghgh55}.
\end{proof}

\subsubsection{Thermodynamical equilibrium; canonical
ensemble} Let $\xi\left(\vec x_1,\ldots,\vec x_n,\vec
y_1,\ldots,\vec y_n,t\right)$ be the density matrix of the quantum
system, ruled by the Hamiltonian operator $\hat H_0$ from
\er{vhfffngghkjgghfjjghghghhjghjgghkghgghjhggjjkgSYSPNjj}, in the
case of \underline{thermodynamical equilibrium} and $\hat R_\xi$ is
given by
\er{gyfyfgfgfgghZZHHhuggkjhjjggjhhhjhgg22jjjkjkkljjjkkjopiiiojjkljkjkjkljkjk}
for this $\xi$. Then in the case that $\frac{\partial \hat
H_0}{\partial t}\equiv 0$ and the given equilibrium system rests
macroscopically, i.e. it has the macroscopical velocity field zero:
$\vec u(\vec x,t)\equiv 0$ it is well known that
\begin{equation}\label{vhfffngghkjgghfjjghghghSYSPNjljllkljjhkhkhklpl;lpoopkklkljlrrcc}
\hat R_\xi=\frac{1}{trace \left( f\circ \hat H_0 \right)} f\circ
\hat H_0,
\end{equation}
and in the case of a system in thermostat, having the Kelvin
temperature $T$, we have the canonical ensemble:
\begin{equation}\label{vhfffngghkjgghfjjghghghSYSPNjljllkljjhkhkhklpl;lpoopkklkljlrrcchhkhcc}
f(s)=e^{-\frac{s}{kT}}=\sum_{m=0}^{+\infty}\frac{1}{m!}\left(-\frac{s}{kT}\right)^m\quad\forall
s\in\mathbb{C},
\end{equation}
where $k$ is the Boltzmann constant and by $f\circ \hat H_0$ we
denote the following operator
\begin{equation}\label{vhfffngghkjgghfjjghghghSYSPNjljllkljjhkhkhklpl;lpoopkklkljlrrcchhkhcchjhjgcc}
f\circ \hat
H_0=\sum_{m=0}^{+\infty}\frac{1}{m!}\left(-\frac{1}{kT}\hat
H_0\right)^m\quad\forall s\in\mathbb{C}.
\end{equation}
We would like to find alternative form of the above law of
thermodynamical equilibrium, which is invariant under the change of
inertial or non-inertial cartesian coordinate system. Then it is
clear that if the given system rests in the old coordinate system,
then it obviously has non-trivial macroscopical velocity field $\vec
u(\vec x,t)$ in the new one. Moreover, $\vec u(\vec x,t)$ can depend
on $\vec x$ and $t$, as it indeed happens in the case of a rotation
of the new coordinate system with respect to the old one. On the
other hand the concept of thermodynamical equilibrium is clearly
independent of the coordinate system.

In order to find the forms of the law of thermodynamical
equilibrium, which are indeed invariant under the change of inertial
or non-inertial cartesian coordinate system we follow the steps as
below. Given a speed-like vector field $\vec u=\vec u(\vec x,t)$,
define
\begin{equation}\label{vhfffngghkjgghfjjghghghSYSPNjljll'ljjjjjjhkjhjhjjkjkkjjz}
\hat H_{\vec u}\left(\vec x_1,\ldots,\vec x_n,t\right):= \hat
H_0\left(\vec x_1,\ldots,\vec
x_n,t\right)-\frac{1}{2}\left(\sum_{j=1}^{n}\hat {\vec P}_j\cdot\vec
u(\vec x_j,t)+\vec u(\vec x_j,t)\cdot \hat {\vec P}_j\right),
\end{equation}
where $\hat{\vec P}_j$ is the operator defined as
\begin{equation}\label{vhfffngghkjgghfjjghghghSYSPNjljll'ljjjjjjhkjhjhjjkjkkljjljjkkjjz}
\hat {\vec P}_j\cdot\psi\left(\vec x_1,\ldots,\vec
x_n\right)=-i\hbar\nabla_{\vec x_j}\psi\left(\vec x_1,\ldots,\vec
x_n\right).
\end{equation}
Then by \er{vhfffngghkjgghfjjghghghhjghjgghkghgghjhggjjkgSYSPNjj} we
have
\begin{multline}\label{vhfffngghkjgghfjjghghghhjghjgghkghgghjhggjjkgSYSPNjjjljjkljjkjkjjz}
\hat H_{\vec u}\left(\vec x_1,\ldots,\vec
x_n,t\right)\cdot\psi\left(\vec x_1,\ldots,\vec x_n,t\right)=
-\sum_{j=1}^{n}\frac{\hbar^2}{2m_j}\Delta_{\vec
x_j}\psi-\sum_{j=1}^{n}\frac{i\hbar}{2}div_{\vec
x_j}\left\{\psi\left(\vec v(\vec x_j,t)-\vec u(\vec
x_j,t)\right)\right\}\\-\sum_{j=1}^{n}\frac{i\hbar}{2}\left(\vec
v(\vec x_j,t)-\vec u(\vec x_j,t)\right)\cdot\nabla_{\vec
x_j}\psi+\sum_{j=1}^{n}\frac{ i\hbar\sigma_j}{2m_jc}div_{\vec
x_j}\left\{\psi\vec A(\vec x_j,t)\right\}+\sum_{j=1}^{n}\frac{
i\hbar\sigma_j}{2m_jc}\vec A(\vec x_j,t)\cdot\nabla_{\vec
x_j}\psi\\+\sum_{j=1}^{n}\left(\sigma_j\Psi(\vec
x_j,t)-\frac{\sigma_j}{c}\vec A(\vec x_j,t)\cdot\vec v(\vec
x_j,t)+\frac{\sigma^2_j}{2m_jc^2}\left|\vec A(\vec
x_j,t)\right|^2\right)\psi-V\left(\vec x_1,\ldots,\vec
x_n,t\right)\psi.
\end{multline}
Thus, as before, we can prove that
\er{vhfffngghkjgghfjjghghghhjghjgghkghgghjhggjjkgSYSPNjjjljjkljjkjkjjz}
is invariant under the change of inertial or non-inertial cartesian
coordinate systems i.e.
\begin{multline}\label{vhfffngghkjgghfjjghghghhjghjgghkghgghjhggjjkgSYSPNjjjljjkljjkjkjjz1155}
\hat H'_{\vec u'}\left(\vec x'_1,\ldots,\vec
x'_n,t'\right)\cdot\psi'\left(\vec x'_1,\ldots,\vec
x'_n,t'\right)=\phi'\left(\vec x'_1,\ldots,\vec x'_n,t'\right)
\quad\text{if and only if}\quad\\  \hat H_{\vec u}\left(\vec
x_1,\ldots,\vec x_n,t\right)\cdot\psi\left(\vec x_1,\ldots,\vec
x_n,t\right)=\phi\left(\vec x_1,\ldots,\vec x_n,t\right),
\end{multline}
provided that, as before, we have
\begin{equation}\label{yuythfgfyftydtydtydtyddyyyhhddhhhredPPN111hgghjgintintrrZZHHhkkhjhj}
\begin{cases}
V'\left(\vec x'_1,\ldots,\vec x'_n,t'\right)\,=\,V\left(\vec
x_1,\ldots,\vec x_n,t\right),\\
\sigma'_j\,=\,\sigma_j,\\
m'_j\,=\,m_j,\\
\vec v'(\vec x',t)\,=\,A(t)\cdot \vec v(\vec x,t)+\frac{dA}{dt}(t)\cdot\vec x+\frac{d\vec z}{dt}(t),\\
\vec u'(\vec x',t)\,=\,A(t)\cdot \vec u(\vec x,t)+\frac{dA}{dt}(t)\cdot\vec x+\frac{d\vec z}{dt}(t),\\
\vec A'(\vec x',t)\,=\,A(t)\cdot \vec A(\vec x,t),\\
\Psi'(\vec x',t)-\vec v'(\vec x',t)\cdot\vec A'(\vec x',t)\,=\,\Psi(\vec x,t)-\vec v(\vec x,t)\cdot\vec A(\vec x,t),\\
\psi'\left(\vec x'_1,\ldots,\vec x'_n,t'\right)=\psi\left(\vec
x_1,\ldots,\vec x_n,t\right),\\
\phi'\left(\vec x'_1,\ldots,\vec x'_n,t'\right)=\phi\left(\vec
x_1,\ldots,\vec x_n,t\right).
\end{cases}
\end{equation}
%
%
%
%
%
%
Next consider
\begin{equation}\label{vhfffngghkjgghfjjghghghSYSPNjljllkljjhkhkhklpl;lpooprrcc}
\hat R_\xi=\frac{1}{trace \left( f\circ \hat H_{\vec u} \right)}
f\circ \hat H_{\vec u},
\end{equation}
where, as before, in the case of a system in thermostat, having the
Kelvin temperature $T$, we have the canonical ensemble:
\begin{equation}\label{vhfffngghkjgghfjjghghghSYSPNjljllkljjhkhkhklpl;lpoopkklkljlrrcchhkhcchhhjjkljcc}
f(s)=e^{-\frac{s}{kT}}=\sum_{m=0}^{+\infty}\frac{1}{m!}\left(-\frac{s}{kT}\right)^m\quad\forall
s\in\mathbb{C}
\end{equation}
with $k$ being the Boltzmann constant and $f\circ \hat H_{\vec u}$
is the following operator
\begin{equation}\label{vhfffngghkjgghfjjghghghSYSPNjljllkljjhkhkhklpl;lpoopkklkljlrrcchhkhcchjhjgcckjhhjcc}
f\circ \hat H_{\vec
u}=\sum_{m=0}^{+\infty}\frac{1}{m!}\left(-\frac{1}{kT}\hat H_{\vec
u}\right)^m\quad\forall s\in\mathbb{C}.
\end{equation}
We would like to note here that since the operator $ \hat H_{\vec
u}$ is bounded from below and $f(s)$ decays rapidly as
$s\to+\infty$, there exists a density matrix $\xi_1$, such that the
operator $\hat R_{\xi_1}$, given by
\er{gyfyfgfgfgghZZHHhuggkjhjjggjhhhjhgg22jjjkjkkljjjkkjopiiiojjkljkjkjkljkjk},
equals to the operator $f\circ \hat H_{\vec u}$ and thus
\er{vhfffngghkjgghfjjghghghSYSPNjljllkljjhkhkhklpl;lpooprrcc} indeed
has sense. Next by
\er{vhfffngghkjgghfjjghghghhjghjgghkghgghjhggjjkgSYSPNjjjljjkljjkjkjjz1155}
and
\er{vhfffngghkjgghfjjghghghSYSPNjljllkljjhkhkhklpl;lpoopkklkljlrrcchhkhcchjhjgcckjhhjcc}
the operator $\hat R_\xi$ in the left hand side of
\er{vhfffngghkjgghfjjghghghSYSPNjljllkljjhkhkhklpl;lpooprrcc} is
invariant under the change of inertial or non-inertial cartesian
coordinate systems i.e.
\begin{multline}\label{vhfffngghkjgghfjjghghghhjghjgghkghgghjhggjjkgSYSPNjjjljjkljjkjkjjz1155kjjkjk}
\hat R_{\xi'}\left(\vec x'_1,\ldots,\vec
x'_n,t'\right)\cdot\psi'\left(\vec x'_1,\ldots,\vec
x'_n,t'\right)=\phi'\left(\vec x'_1,\ldots,\vec x'_n,t'\right)
\quad\text{if and only if}\quad\\  \hat R_\xi\left(\vec
x_1,\ldots,\vec x_n,t\right)\cdot\psi\left(\vec x_1,\ldots,\vec
x_n,t\right)=\phi\left(\vec x_1,\ldots,\vec x_n,t\right),
\end{multline}
provided that, as before, we have
\er{yuythfgfyftydtydtydtyddyyyhhddhhhredPPN111hgghjgintintrrZZHHhkkhjhj}
and $\xi'=\xi$. Moreover, in the case $\vec u\equiv 0$
\er{vhfffngghkjgghfjjghghghSYSPNjljllkljjhkhkhklpl;lpooprrcc}
coincides with
\er{vhfffngghkjgghfjjghghghSYSPNjljllkljjhkhkhklpl;lpoopkklkljlrrcc}.
However, we still need to derive the restrictions on the field $\vec
u$ and the Hamiltonian operator $\hat H_0$, providing that our
system can indeed be found in the state of thermodynamical
equilibrium. We remind that in the case $\vec u\equiv 0$ the
appropriate restriction is $\frac{\partial \hat H_0}{\partial
t}\equiv 0$. In order to get these restrictions in the general case,
we need to insert $\hat R_\xi$ in
\er{vhfffngghkjgghfjjghghghSYSPNjljllkljjhkhkhklpl;lpooprrcc} into
the equation in
\er{vhfffngghkjgghfjjghghghhjghjgghkghggkghghjghSYSPNjjjhhhjjkjkjklkjljl}
which is equivalent to the quantum Liouville equation.

Therefore, assume that the vector field $\vec u$ satisfies
\er{vhfffngghkjgghfjjghghghhjghjgghkghggSYSPNkljljjkklmjkjkjjk,mmm,kl;k;lklkjkjkkjkjkjkjljljjljklggghgkhhkhkhkkklkklklkljkhkhjjkhjh}
with \er{vhfffngghkjgghfjjghghghSYSPNjljll'lhkhhjgggjgjgghgjkjjk}
i.e.
\begin{equation}\label{vhfffngghkjgghfjjghghghhjghjgghkghggSYSPNkljljjkklmjkjkjjk,mmm,kl;k;lklkjkjkkjkjkjkjljljjljklggghgkhhkhkhkkklkklklkljkhkhjjkhjhmnnmmn}
\begin{cases}
\frac{\partial U}{\partial t}\left(\vec x_1,\ldots,\vec
x_n,t\right)+\sum_{j=1}^{n}\vec u( \vec x_j,t)\cdot\nabla_{\vec
x_j}U\left(\vec x_1,\ldots,\vec x_n,t\right)=0,\\
\frac{\partial}{\partial t}\left(\vec u(\vec x,t)-\vec v(\vec
x,t)\right)-curl_{\vec x}\left(\vec u(\vec x,t)\times\left(\vec
u(\vec x,t)-\vec v(\vec x,t)\right)\right)+\left(\Div_{\vec
x}\left(\vec u(\vec x,t)-\vec v(\vec x,t)\right)\right)\vec
u(\vec x,t)=0,\\
\frac{\partial\vec A}{\partial t}(\vec x,t)-curl_{\vec x}\left(\vec
u(\vec x,t)\times\vec A(\vec x,t)\right)+\left(\Div_{\vec x}\vec
A(\vec x,t)\right)\vec u(\vec x,t)=0,
\\
d_{\vec x}\vec u(\vec x,t)+\left\{d_{\vec x}\vec u(\vec
x,t)\right\}^T=0,
\end{cases}
\end{equation}
where
\begin{multline}\label{vhfffngghkjgghfjjghghghSYSPNjljll'lhkhhjgggjgjgghgjkjjkjkjk}
U\left(\vec x_1,\ldots,\vec x_n,t\right):=\\
\sum_{j=1}^{n}\left(\frac{\sigma^2_j}{2m_jc^2}\left|\vec A(\vec
x_j,t)\right|^2+\sigma_j\Psi(\vec x_j,t)-\frac{\sigma_j}{c}\vec
A(\vec x_j,t)\cdot\vec v(\vec x_j,t)\right)-V\left(\vec
x_1,\ldots,\vec x_n,t\right).
\end{multline}
Then, by Proposition \ref{eqv1} below we have
\begin{equation}\label{vhfffngghkjgghfjjghghghhjghjgghkghggkghghjghSYSPNjjjhhhjjkjkjklkjljljkkhhkjiijjijhbmbgghghcchkhh}
i\hbar\frac{\partial}{\partial t}\left(f\circ \hat H_{\vec
u}\right)= \hat H_0\cdot\left(f\circ \hat H_{\vec
u}\right)-\left(f\circ \hat H_{\vec u}\right)\cdot\hat H_0,
\end{equation}
Then by
\er{vhfffngghkjgghfjjghghghhjghjgghkghggkghghjghSYSPNjjjhhhjjkjkjklkjljljkkhhkjiijjijhbmbgghghcchkhh}
together with
\er{gyfyfgfgfgghZZHHhuggkjhjjgg22hjhjjhjhmnbmbjkjhjhjhhhjhj} we
deduce that
\begin{equation}\label{vhfffngghkjgghfjjghghghhjghjgghkghggkghghjghSYSPNjjjhhhjjkjkjklkjljljkkhhkjiijjijhbmbgghghcchkhhjkjkk}
i\hbar\frac{\partial}{\partial t}\left(\hat R_\xi\right)= \hat
H_0\cdot\left(\hat R_\xi\right)-\left(\hat R_\xi\right)\cdot\hat
H_0,
\end{equation}
where $\hat R_\xi$ is the operator in the left hand side of
\er{vhfffngghkjgghfjjghghghSYSPNjljllkljjhkhkhklpl;lpooprrcc}. So we
indeed get
\er{vhfffngghkjgghfjjghghghhjghjgghkghggkghghjghSYSPNjjjhhhjjkjkjklkjljl}
in the case of
\er{vhfffngghkjgghfjjghghghhjghjgghkghggSYSPNkljljjkklmjkjkjjk,mmm,kl;k;lklkjkjkkjkjkjkjljljjljklggghgkhhkhkhkkklkklklkljkhkhjjkhjhmnnmmn},
\er{vhfffngghkjgghfjjghghghSYSPNjljll'lhkhhjgggjgjgghgjkjjkjkjk}.
Moreover,
\er{vhfffngghkjgghfjjghghghhjghjgghkghggSYSPNkljljjkklmjkjkjjk,mmm,kl;k;lklkjkjkkjkjkjkjljljjljklggghgkhhkhkhkkklkklklkljkhkhjjkhjhmnnmmn},
\er{vhfffngghkjgghfjjghghghSYSPNjljll'lhkhhjgggjgjgghgjkjjkjkjk} are
invariant under the change of inertial or non-inertial coordinate
system.
\begin{proposition}\label{eqv1}
Assume that the speed-like vector field $\vec u$ satisfies
\er{vhfffngghkjgghfjjghghghhjghjgghkghggSYSPNkljljjkklmjkjkjjk,mmm,kl;k;lklkjkjkkjkjkjkjljljjljklggghgkhhkhkhkkklkklklkljkhkhjjkhjhmnnmmn}
and
\er{vhfffngghkjgghfjjghghghSYSPNjljll'lhkhhjgggjgjgghgjkjjkjkjk}.
Next, assume that the holomorphic function $f(z):\mathbb{C}\to
\mathbb{C}$ defined as a sum of the power series
\begin{equation}\label{vhfffngghkjgghfjjghghghhjghjgghkghggkghghjghSYSPNjjjhhhjjkjkjklkjljljkkhhkjiijjihhjjhjhjhjhkkhhkcc}
f(z):=\sum_{m=0}^{+\infty}a_mz^m,
\end{equation}
with $a_m\in\mathbb{C}$, is such that for the operator $f\circ \hat
H_{\vec u}$, given by:
\begin{equation}\label{vhfffngghkjgghfjjghghghhjghjgghkghggkghghjghSYSPNjjjhhhjjkjkjklkjljljkkhhkjiijjihhjjhjhjhjhkkhhknnhkhhcc}
f\circ \hat H_{\vec u}:=\sum_{m=0}^{+\infty}a_m\hat H^m_{\vec u},
\end{equation}
where the operator $\hat H_{\vec u}$  is given by
\er{vhfffngghkjgghfjjghghghhjghjgghkghgghjhggjjkgSYSPNjjjljjkljjkjkjjz},
there exists a density matrix $\xi\left(\vec x_1,\ldots,\vec
x_n,\vec y_1,\ldots,\vec y_n,t\right)$, such that
\begin{equation}\label{vhfffngghkjgghfjjghghghhjghjgghkghggkghghjghSYSPNjjjhhhjjkjkjklkjljljkkhhkjiijjijhbmbgghgh11ggghhgfcc}
\left(f\circ\hat H_{\vec u}\right)\left(\vec x_1,\ldots,\vec
x_n,t\right)=\hat R_{\xi}\left(\vec x_1,\ldots,\vec x_n,t\right),
\end{equation}
where $\hat R_\xi$ is given by
\er{gyfyfgfgfgghZZHHhuggkjhjjggjhhhjhgg22jjjkjkkljjjkkjopiiiojjkljkjkjkljkjk}.
Then the operator $f\circ \hat H_{\vec u}$ satisfies:
\begin{equation}\label{vhfffngghkjgghfjjghghghhjghjgghkghggkghghjghSYSPNjjjhhhjjkjkjklkjljljkkhhkjiijjijhbmbgghghcc}
i\hbar\frac{\partial}{\partial t}\left(f\circ \hat H_{\vec
u}\right)= \hat H_0\cdot\left(f\circ \hat H_{\vec
u}\right)-\left(f\circ \hat H_{\vec u}\right)\cdot\hat H_0,
\end{equation}
or equivalently
\begin{multline}\label{vhfffngghkjgghfjjghghghhjghjgghkghggkghghjghSYSPNjjjhhhcc}
i\hbar\frac{\partial\xi}{\partial t}\left(\vec x_1,\ldots,\vec
x_n,\vec y_1,\ldots,\vec y_n,t\right)=\\ \hat H_0\left(\vec
x_1,\ldots,\vec x_n,t\right)\cdot\xi\left(\vec x_1,\ldots,\vec
x_n,\vec y_1,\ldots,\vec y_n,t\right)-\hat H^*_0\left(\vec
y_1,\ldots,\vec y_n,t\right)\cdot\xi\left(\vec x_1,\ldots,\vec
x_n,\vec y_1,\ldots,\vec y_n,t\right).
\end{multline}
Moreover, $f\circ \hat H_{\vec u}$ is invariant under the change of
inertial or non-inertial cartesian coordinate systems i.e.
\begin{multline}\label{vhfffngghkjgghfjjghghghhjghjgghkghgghjhggjjkgSYSPNjjjljjkljjkjkjjzjjjkkhhklhhlcc}
\left(f\circ\hat H'_{\vec u'}\right)\left(\vec x'_1,\ldots,\vec
x'_n,t'\right)\cdot\psi'\left(\vec x'_1,\ldots,\vec
x'_n,t'\right)=\phi'\left(\vec x'_1,\ldots,\vec x'_n,t'\right)
\quad\text{if and only if}\quad\\  \left(f\circ\hat H_{\vec
u}\right)\left(\vec x_1,\ldots,\vec x_n,t\right)\cdot\psi\left(\vec
x_1,\ldots,\vec x_n,t\right)=\phi\left(\vec x_1,\ldots,\vec
x_n,t\right),
\end{multline}
provided that we have
\er{yuythfgfyftydtydtydtyddyyyhhddhhhredPPN111hgghjgintintrrZZHHhkkhjhj}.
Finally, if we assume that the first and the second particles have
the same mass $m_1=m_2$ and the same charge $\sigma_1=\sigma_2$ and
moreover, if we assume that
$$V\left(\vec x_1,\vec x_2,\vec x_3,\ldots,\vec
x_n,t\right)=V\left(\vec x_2,\vec x_1,\vec x_3,\ldots,\vec
x_n,t\right)\quad\quad\forall\,\vec x_1,\vec x_2,\vec
x_3,\ldots,\vec x_n\in\mathbb{R}^3,\quad\forall t,$$ then
\begin{multline}\label{fyfgjguyiuiHHgjgjgghbmgjgfhfhhHHljjjhjhjhyyuyuhhhljjjkjjcc}
\psi(\vec x_2,\vec x_1,\vec x_3,\ldots,\vec x_n,t)=-\psi(\vec
x_1,\vec x_2,\vec x_3,\ldots,\vec x_n,t)\;\;\forall\,\vec
x_1,\ldots,\vec
x_n\in\mathbb{R}^3\quad\text{and}\quad \phi=\left(f\circ\hat H_{\vec u}\right)\cdot\psi\\
\quad\text{implies}\quad \phi(\vec x_2,\vec x_1,\vec x_3,\ldots,\vec
x_n,t)=-\phi(\vec x_1,\vec x_2,\vec x_3,\ldots,\vec
x_n,t)\;\;\forall\,\vec x_1,\ldots,\vec x_n\in\mathbb{R}^3,
\end{multline}
and
\begin{multline}\label{fyfgjguyiuiHHgjgjgghbmgjgfhfhhHHljjjhjhjhyyuyuhhhljjjkjjggcc}
\psi(\vec x_2,\vec x_1,\vec x_3,\ldots,\vec x_n,t)=\psi(\vec
x_1,\vec x_2,\vec x_3,\ldots,\vec x_n,t)\;\;\forall\,\vec
x_1,\ldots,\vec
x_n\in\mathbb{R}^3\quad\text{and}\quad \phi=\left(f\circ\hat H_{\vec u}\right)\cdot\psi\\
\quad\text{implies}\quad \phi(\vec x_2,\vec x_1,\vec x_3,\ldots,\vec
x_n,t)=\phi(\vec x_1,\vec x_2,\vec x_3,\ldots,\vec
x_n,t)\;\;\forall\,\vec x_1,\ldots,\vec x_n\in\mathbb{R}^3.
\end{multline}
\end{proposition}
\begin{proof}
Again by
\er{vhfffngghkjgghfjjghghghhjghjgghkghggSYSPNkljljjkklmjkjkjjk,mmm,kl;k;lklkjkjkkjkjkjkjljljjljklggghgkhhkhkhkkklkklklkljkhkhjjkhjhmnnmmn}
and Proposition \ref{jhjkhhjhjhgjgjjkjk} there exists another
cartesian coordinate system $(**)$ such that under the change of
coordinate system $(*)$ to another cartesian coordinate system
$(**)$, given by \er{noninchgravortbstrjgghguittu2hhhhh}, we have
\begin{equation}
\label{jljjjjkhhhjkkww}A(t)\cdot \vec u(\vec x,t)+
\frac{d A}{dt}(t)\cdot\vec x+
\frac{d\vec z}{dt}(t)=\vec u'(\vec x',t')=0.
\end{equation}
Thus, since
\er{vhfffngghkjgghfjjghghghhjghjgghkghggSYSPNkljljjkklmjkjkjjk,mmm,kl;k;lklkjkjkkjkjkjkjljljjljklggghgkhhkhkhkkklkklklkljkhkhjjkhjhmnnmmn},
\er{vhfffngghkjgghfjjghghghSYSPNjljll'lhkhhjgggjgjgghgjkjjkjkjk} are
invariant under the change of inertial or non-inertial cartesian
coordinate systems, as before, in system $(**)$ we have
\begin{equation}\label{vhfffngghkjgghfjjghghghhjghjgghkghggSYSPNkljljjkklmjkjkjjk,mmm,kl;k;lklkjkjkkjkjkjkjljljjljklggghgkhhkhkhkkklkklklkljkhkhjjkhjhmjmm,mjhjkjkhjhj}
\begin{cases}
\frac{\partial U'}{\partial t'}\left(\vec x'_1,\ldots,\vec
x'_n,t'\right)=0,\\
\frac{\partial\vec v'}{\partial t'}(\vec
x',t')=0,\\
\frac{\partial\vec A'}{\partial t'}(\vec x',t')=0
\\
\vec u'(\vec x',t')=0,
\end{cases}
\end{equation}
where
\begin{multline}\label{vhfffngghkjgghfjjghghghSYSPNjljll'lhkhhjgggjgjgghgjkjjkkjjhj}
U'\left(\vec x'_1,\ldots,\vec x'_n,t'\right):=\\
\sum_{j=1}^{n}\left(\frac{(\sigma'_j)^2}{2m'_jc^2}\left|\vec A'(\vec
x'_j,t')\right|^2+\sigma'_j\Psi'(\vec
x'_j,t')-\frac{\sigma'_j}{c}\vec A'(\vec x'_j,t')\cdot\vec v'(\vec
x'_j,t')\right)-V'\left(\vec x'_1,\ldots,\vec x'_n,t'\right).
\end{multline}
On the other hand,
\er{vhfffngghkjgghfjjghghghhjghjgghkghggSYSPNkljljjkklmjkjkjjk,mmm,kl;k;lklkjkjkkjkjkjkjljljjljklggghgkhhkhkhkkklkklklkljkhkhjjkhjhmjmm,mjhjkjkhjhj}
is equivalent to
\begin{equation}\label{vhfffnhhjhjjkkjhkjkhkh}
\frac{\partial \hat H'_0}{\partial t'}\left(\vec x'_1,\ldots,\vec
x'_n,t'\right)=0\quad\text{and}\quad\vec u'(\vec x',t')=0,
\end{equation}
in system $(**)$. Thus, by \er{vhfffnhhjhjjkkjhkjkhkh}, in system
$(**)$ we deduce:
\begin{equation}\label{vhfffngghkjgghfjjghghghhjghjgghkghggkghghjghSYSPNjjjhhhjjkjkjklkjljljkkhhkjjjkljjbbb}
\hat H'_0\cdot\hat H'_{\vec u'}-\hat H'_{\vec u'}\cdot\hat H'_0=\hat
H'_0\cdot\hat H'_0-\hat H'_0\cdot\hat
H'_0=0=i\hbar\frac{\partial\hat H'_0}{\partial
t'}=i\hbar\frac{\partial\hat H'_{\vec u'}}{\partial t'}.
\end{equation}
So we get:
\begin{equation}\label{vhfffngghkjgghfjjghghghhjghjgghkghggkghghjghSYSPNjjjhhhjjkjkjklkjljljkkhhkjjjkljjbbbjkljljl}
i\hbar\frac{\partial\hat H'_{\vec u'}}{\partial t'}=\hat
H'_0\cdot\hat H'_{\vec u'}-\hat H'_{\vec u'}\cdot\hat H'_0.
\end{equation}
Therefore, given the holomorphic function $f(z):\mathbb{C}\to
\mathbb{C}$ defined as a sum of the power series
\begin{equation}\label{vhfffngghkjgghfjjghghghhjghjgghkghggkghghjghSYSPNjjjhhhjjkjkjklkjljljkkhhkjiijjihhjjhjhjhjhkkhhk}
f(z):=\sum_{m=0}^{+\infty}a_mz^m,
\end{equation}
with $a_m\in\mathbb{C}$, if we define the operator $f\circ \hat
H_{\vec u}$ as:
\begin{equation}\label{vhfffngghkjgghfjjghghghhjghjgghkghggkghghjghSYSPNjjjhhhjjkjkjklkjljljkkhhkjiijjihhjjhjhjhjhkkhhknnhkhh}
f\circ \hat H_{\vec u}:=\sum_{m=0}^{+\infty}a_m\hat H^m_{\vec u},
\end{equation}
by
\er{vhfffngghkjgghfjjghghghhjghjgghkghggkghghjghSYSPNjjjhhhjjkjkjklkjljljkkhhkjjjkljjbbbjkljljl}
and Proposition \ref{hugyfyfffgh} we deduce
\begin{equation}\label{vhfffngghkjgghfjjghghghhjghjgghkghggkghghjghSYSPNjjjhhhjjkjkjklkjljljkkhhkjiijjijhbmbgghghdd}
i\hbar\frac{\partial}{\partial t'}\left(f\circ \hat H'_{\vec
u'}\right)= \hat H'_0\cdot\left(f\circ \hat H'_{\vec
u'}\right)-\left(f\circ \hat H'_{\vec u'}\right)\cdot\hat H'_0.
\end{equation}
Moreover, by
\er{vhfffngghkjgghfjjghghghhjghjgghkghgghjhggjjkgSYSPNjjjljjkljjkjkjjz}
and
\er{vhfffngghkjgghfjjghghghhjghjgghkghggkghghjghSYSPNjjjhhhjjkjkjklkjljljkkhhkjiijjihhjjhjhjhjhkkhhknnhkhh},
we can easily prove that $f\circ \hat H_{\vec u}$ is invariant under
the change of inertial or non-inertial cartesian coordinate systems
i.e.
\begin{multline}\label{vhfffngghkjgghfjjghghghhjghjgghkghgghjhggjjkgSYSPNjjjljjkljjkjkjjzjjjkkhhklhhl}
\left(f\circ\hat H'_{\vec u'}\right)\left(\vec x'_1,\ldots,\vec
x'_n,t'\right)\cdot\psi'\left(\vec x'_1,\ldots,\vec
x'_n,t'\right)=\phi'\left(\vec x'_1,\ldots,\vec x'_n,t'\right)
\quad\text{if and only if}\quad\\  \left(f\circ\hat H_{\vec
u}\right)\left(\vec x_1,\ldots,\vec x_n,t\right)\cdot\psi\left(\vec
x_1,\ldots,\vec x_n,t\right)=\phi\left(\vec x_1,\ldots,\vec
x_n,t\right),
\end{multline}
provided that, as before, we have
\er{yuythfgfyftydtydtydtyddyyyhhddhhhredPPN111hgghjgintintrrZZHHhkkhjhj}.
Next assume that the holomorphic function $f$ in
\er{vhfffngghkjgghfjjghghghhjghjgghkghggkghghjghSYSPNjjjhhhjjkjkjklkjljljkkhhkjiijjihhjjhjhjhjhkkhhk}
is such that there exists a density matrix $\xi\left(\vec
x_1,\ldots,\vec x_n,\vec y_1,\ldots,\vec y_n,t\right)$ satisfying
\begin{equation}\label{vhfffngghkjgghfjjghghghhjghjgghkghggkghghjghSYSPNjjjhhhjjkjkjklkjljljkkhhkjiijjijhbmbgghgh11ggghhgf11gg}
\left(f\circ\hat H_{\vec u}\right)\left(\vec x_1,\ldots,\vec
x_n,t\right)=\hat R_{\xi}\left(\vec x_1,\ldots,\vec x_n,t\right).
\end{equation}
Then, by
\er{vhfffngghkjgghfjjghghghhjghjgghkghgghjhggjjkgSYSPNjjjljjkljjkjkjjzjjjkkhhklhhl},
for the density matrix $\xi'\left(\vec x'_1,\ldots,\vec x'_n,\vec
y'_1,\ldots,\vec y'_n,t'\right)$ we have
\begin{equation}\label{vhfffngghkjgghfjjghghghhjghjgghkghggkghghjghSYSPNjjjhhhjjkjkjklkjljljkkhhkjiijjijhbmbgghgh11ggghhgf}
\left(f\circ\hat H'_{\vec u'}\right)\left(\vec x'_1,\ldots,\vec
x'_n,t'\right)=\hat R_{\xi'}\left(\vec x'_1,\ldots,\vec
x'_n,t'\right),
\end{equation}
provided that $\xi'=\xi$.
Therefore, since we obtained before, that equation
\er{vhfffngghkjgghfjjghghghhjghjgghkghggkghghjghSYSPNjjjhhhjjkjkjklkjljljkkhhkjiijjijhbmbgghghdd}
is equivalent to the primed version of
\er{vhfffngghkjgghfjjghghghhjghjgghkghggkghghjghSYSPNjjjhhh} and at
the same time equation
\er{vhfffngghkjgghfjjghghghhjghjgghkghggkghghjghSYSPNjjjhhh} is
invariant under the change of non-inertial cartesian coordinate
system, provided we have $\xi'=\xi$, with the help of
\er{vhfffngghkjgghfjjghghghhjghjgghkghggkghghjghSYSPNjjjhhhjjkjkjklkjljljkkhhkjiijjijhbmbgghgh11ggghhgf11gg},
as before, we deduce
\begin{equation}\label{vhfffngghkjgghfjjghghghhjghjgghkghggkghghjghSYSPNjjjhhhjjkjkjklkjljljkkhhkjiijjijhbmbgghgh}
i\hbar\frac{\partial}{\partial t}\left(f\circ \hat H_{\vec
u}\right)= \hat H_0\cdot\left(f\circ \hat H_{\vec
u}\right)-\left(f\circ \hat H_{\vec u}\right)\cdot\hat H_0
\end{equation}
in an arbitrary coordinate system. So we obtain
\er{vhfffngghkjgghfjjghghghhjghjgghkghggkghghjghSYSPNjjjhhhjjkjkjklkjljljkkhhkjiijjijhbmbgghghcc}
or equivalently
\er{vhfffngghkjgghfjjghghghhjghjgghkghggkghghjghSYSPNjjjhhhcc}.

Finally, assume that the first and the second particles have the
same mass $m_1=m_2$ and the same charge $\sigma_1=\sigma_2$ and
moreover, assume that we have
$$V\left(\vec x_1,\vec x_2,\vec x_3,\ldots,\vec
x_n,t\right)=V\left(\vec x_2,\vec x_1,\vec x_3,\ldots,\vec
x_n,t\right)\quad\quad\forall\,\vec x_1,\vec x_2,\vec
x_3,\ldots,\vec x_n\in\mathbb{R}^3,\quad\forall t.$$ In this case
clearly by
\er{vhfffngghkjgghfjjghghghhjghjgghkghgghjhggjjkgSYSPNjjjljjkljjkjkjjz}
we deduce the following relations:
\begin{multline}\label{fyfgjguyiuiHHgjgjgghbmgjgfhfhhHHljjjhjhjhyyuyuhhh11}
\psi(\vec x_2,\vec x_1,\vec x_3,\ldots,\vec x_n,t)=-\psi(\vec
x_1,\vec x_2,\vec x_3,\ldots,\vec x_n,t)\;\;\forall\,\vec
x_1,\ldots,\vec
x_n\in\mathbb{R}^3\quad\text{and}\quad \phi=\hat H_{\vec u}\cdot\psi\\
\quad\text{implies}\quad \phi(\vec x_2,\vec x_1,\vec x_3,\ldots,\vec
x_n,t)=-\phi(\vec x_1,\vec x_2,\vec x_3,\ldots,\vec
x_n,t)\;\;\forall\,\vec x_1,\ldots,\vec x_n\in\mathbb{R}^3,
\end{multline}
and
\begin{multline}\label{fyfgjguyiuiHHgjgjgghbmgjgfhfhhHHljjjhjhjhyyuyuhhh11gg}
\psi(\vec x_2,\vec x_1,\vec x_3,\ldots,\vec x_n,t)=\psi(\vec
x_1,\vec x_2,\vec x_3,\ldots,\vec x_n,t)\;\;\forall\,\vec
x_1,\ldots,\vec
x_n\in\mathbb{R}^3\quad\text{and}\quad \phi=\hat H_{\vec u}\cdot\psi\\
\quad\text{implies}\quad \phi(\vec x_2,\vec x_1,\vec x_3,\ldots,\vec
x_n,t)=\phi(\vec x_1,\vec x_2,\vec x_3,\ldots,\vec
x_n,t)\;\;\forall\,\vec x_1,\ldots,\vec x_n\in\mathbb{R}^3.
\end{multline}
Therefore, by
\er{vhfffngghkjgghfjjghghghhjghjgghkghggkghghjghSYSPNjjjhhhjjkjkjklkjljljkkhhkjiijjihhjjhjhjhjhkkhhknnhkhh}
and \er{fyfgjguyiuiHHgjgjgghbmgjgfhfhhHHljjjhjhjhyyuyuhhh11},
\er{fyfgjguyiuiHHgjgjgghbmgjgfhfhhHHljjjhjhjhyyuyuhhh11gg} we obtain
\begin{multline}\label{fyfgjguyiuiHHgjgjgghbmgjgfhfhhHHljjjhjhjhyyuyuhhhljjjkjj11}
\psi(\vec x_2,\vec x_1,\vec x_3,\ldots,\vec x_n,t)=-\psi(\vec
x_1,\vec x_2,\vec x_3,\ldots,\vec x_n,t)\;\;\forall\,\vec
x_1,\ldots,\vec
x_n\in\mathbb{R}^3\quad\text{and}\quad \phi=\left(f\circ\hat H_{\vec u}\right)\cdot\psi\\
\quad\text{implies}\quad \phi(\vec x_2,\vec x_1,\vec x_3,\ldots,\vec
x_n,t)=-\phi(\vec x_1,\vec x_2,\vec x_3,\ldots,\vec
x_n,t)\;\;\forall\,\vec x_1,\ldots,\vec x_n\in\mathbb{R}^3,
\end{multline}
and
\begin{multline}\label{fyfgjguyiuiHHgjgjgghbmgjgfhfhhHHljjjhjhjhyyuyuhhhljjjkjjgg}
\psi(\vec x_2,\vec x_1,\vec x_3,\ldots,\vec x_n,t)=\psi(\vec
x_1,\vec x_2,\vec x_3,\ldots,\vec x_n,t)\;\;\forall\,\vec
x_1,\ldots,\vec
x_n\in\mathbb{R}^3\quad\text{and}\quad \phi=\left(f\circ\hat H_{\vec u}\right)\cdot\psi\\
\quad\text{implies}\quad \phi(\vec x_2,\vec x_1,\vec x_3,\ldots,\vec
x_n,t)=\phi(\vec x_1,\vec x_2,\vec x_3,\ldots,\vec
x_n,t)\;\;\forall\,\vec x_1,\ldots,\vec x_n\in\mathbb{R}^3,
\end{multline}
consistently with
\er{fyfgjguyiuiHHgjgjgghbmgjgfhfhhHHljjjhjhjhyyuyuihhjjk} and
\er{fyfgjguyiuiHHgjgjgghbmgjgfhfhhHHljjjhjhjhyyuyuihhjhj}.
\end{proof}

\subsection{Shr\"{o}dinger-Pauli equation for a spin-half quantum
particle}\label{hjjghjghggh} Consider the motion of a spin-half
quantum micro-particle with inertial mass $m$ and the charge
$\sigma$ in the outer gravitational and electromagnetical field with
characteristics $\vec v(\vec x,t)$, $\vec A(\vec x,t)$ and
$\Psi(\vec x,t)$ and additional conservative field with potential
$V(\vec y,t)$. Since the Hamiltonian for a macro-particle has the
form
\begin{multline}\label{vhfffngghkjgghfjjghghghSYShmyuuiiuuhmhmiopoopnnkkllkkkZZ}
H_{\text{macro}}\left(\vec P,\vec r,t\right)=\\
\frac{m}{2}\left|\vec P-\frac{\sigma}{c}\vec A(\vec r,t)\right|^2
+\sigma\left(\Psi(\vec r,t)-\frac{1}{c}\vec v(\vec r,t)\cdot\vec
A(\vec r,t)\right)-V\left(\vec r,t\right)+\frac{1}{2}\vec v(\vec
r,t)\cdot\vec P+\frac{1}{2}\vec P\cdot\vec v(\vec r,t),
\end{multline}
we built the Hermitian Hamiltonian operator, taking into account the
spin interaction as
\begin{multline}\label{vhfffngghkjgghfjjghghghSYShmyuuiiuuhmhmiopoopnniukjhjkkkllkkkZZ}
\hat H_0\cdot\psi= -\frac{\hbar^2}{2m}\Delta_{\vec
x}\psi+\frac{i\hbar\sigma}{2mc}div_{\vec x}\left\{\psi\vec A(\vec
x,t)\right\}+\frac{i\hbar\sigma}{2mc}\nabla_{\vec x}\psi\cdot\vec
A(\vec x,t)+\frac{\sigma^2}{2mc^2}\left|\vec A(\vec
x,t)\right|^2\psi\\+\sigma\left(\Psi(\vec x,t)-\frac{1}{c}\vec
v(\vec x,t)\cdot\vec A(\vec x,t)\right)\psi-V\left(\vec
x,t\right)\psi-\frac{i\hbar}{2}div_{\vec x}\left\{\psi\vec v(\vec
x,t)\right\}-\frac{i\hbar}{2}\nabla_{\vec x}\psi\cdot\vec v(\vec
x,t)
\\-\frac{g\sigma\hbar}{2mc}\vec
S\cdot\left(curl_{\vec x}\vec A(\vec
x,t)\,\psi\right)+\frac{\hbar}{4}\vec S\cdot\left(curl_{\vec x}\vec
v(\vec x,t)\,\psi\right),
\end{multline}
where $\psi(\vec x,t)=\left(\psi_1(\vec x,t),\psi_2(\vec
x,t)\right)\in\mathbb{C}^2$ is a two-component wave function, $\hat
H_0$ is the Hamiltonian operator, $\vec S:=(S_1,S_2,S_3)$,
$$S_1=\left(\begin{matrix}0&1\\1&0\end{matrix}\right),\quad
S_2=\left(\begin{matrix}0&-i\\i&0\end{matrix}\right)\quad
S_3=\left(\begin{matrix}1&0\\0&-1\end{matrix}\right)$$ are Pauli
matrices and $g$ is a constant that depends on the type of the
particle (for electron we have $g=1$). Note that, in addition to the
classical term of the spin-magnetic interaction, we added another
term to the Hamiltonian, namely $\frac{\hbar}{4}\vec
S\cdot\left(curl_{\vec x}\vec v(\vec x,t)\,\psi\right)$. In the case
of the Newtonian-type gravity, this term vanishes in every
non-rotating and, in particular, in every inertial coordinate
system, however it provides the invariance of the
Shr\"{o}dinger-Pauli equation, under the change of non-inertial
cartesian coordinate system, as can be seen in the following Theorem
\ref{gjghghgghgintintrrZZ}. The Shr\"{o}dinger-Pauli equation for
this particle is
\begin{equation}\label{vhfffngghkjgghfjjghghghhjghjgghkghggkghghjghSYSPNnnkkllkkkZZ}
i\hbar\frac{\partial\psi}{\partial t}=\hat H_0\cdot\psi.
\end{equation}
 Thus,
\begin{multline}\label{vhfffngghkjgghfjjghghghSYShmyuuiiuuhmhmiopoopnniukjhjkk;l;lkhjjkihjjhkkkkkZZ}
i\hbar\frac{\partial\psi}{\partial
t}=-\frac{\hbar^2}{2m}\Delta_{\vec
x}\psi+\frac{i\hbar\sigma}{2mc}div_{\vec x}\left\{\psi\vec A(\vec
x,t)\right\}+\frac{i\hbar\sigma}{2mc}\nabla_{\vec x}\psi\cdot\vec
A(\vec x,t)+\frac{\sigma^2}{2mc^2}\left|\vec A(\vec
x,t)\right|^2\psi\\
%
%
%
+\sigma\left(\Psi(\vec x,t)-\frac{1}{c}\vec v(\vec x,t)\cdot\vec
A(\vec x,t)\right)\psi-V\left(\vec
x,t\right)\psi-\frac{i\hbar}{2}div_{\vec x}\left\{\psi\vec v(\vec
x,t)\right\}-\frac{i\hbar}{2}\nabla_{\vec x}\psi\cdot\vec v(\vec
x,t)\\+\frac{\hbar}{2}\vec S\cdot\left(\left(\frac{1}{2}curl_{\vec
x}\vec v(\vec x,t)-\frac{g\sigma}{mc}curl_{\vec x}\vec A(\vec
x,t)\right)\,\psi\right).
\end{multline}
I.e.
\begin{multline}\label{vhfffngghkjgghfjjghghghSYShmyuuiiuuhmhmiopoopnniukjhjkk;l;lkhjjkihjjhkkkkkjjjZZ}
i\hbar\left(\frac{\partial\psi}{\partial t}+\frac{1}{2}div_{\vec
x}\left\{\psi\vec v(\vec x,t)\right\}+\frac{1}{2}\nabla_{\vec
x}\psi\cdot\vec v(\vec x,t)\right)\\=-\frac{\hbar^2}{2m}\Delta_{\vec
x}\psi+\frac{i\hbar\sigma}{2mc}div_{\vec x}\left\{\psi\vec A(\vec
x,t)\right\}+\frac{i\hbar\sigma}{2mc}\nabla_{\vec x}\psi\cdot\vec
A(\vec x,t)+\frac{\sigma^2}{2mc^2}\left|\vec A(\vec
x,t)\right|^2\psi\\
%
%
%
+\sigma\left(\Psi(\vec x,t)-\frac{1}{c}\vec v(\vec x,t)\cdot\vec
A(\vec x,t)\right)\psi-V\left(\vec
x,t\right)\psi+\frac{\hbar}{2}\vec
S\cdot\left(\left(\frac{1}{2}curl_{\vec x}\vec v(\vec
x,t)-\frac{g\sigma}{mc}curl_{\vec x}\vec A(\vec
x,t)\right)\,\psi\right).
\end{multline}
\begin{theorem}\label{gjghghgghgintintrrZZ}
Consider that the change of some cartesian coordinate system $(*)$
to another cartesian coordinate system $(**)$ is given by
\begin{equation}\label{noninchgravortbstrjgghguittu2intrrrZZ}
\begin{cases}
\vec x'=A(t)\cdot\vec x+\vec z(t),\\
t'=t,
\end{cases}
\end{equation}
where $A(t)\in SO(3)$ is a rotation. Next, assume that in the
coordinate system $(**)$ we observe a validity of the
Shr\"{o}dinger-Pauli equation of the form:
\begin{multline}\label{MaxVacFull1ninshtrredPPNintrrZZ}
i\hbar\left(\frac{\partial\psi'}{\partial t'}+\frac{1}{2}div_{\vec
x'}\left\{\psi'\vec v'\right\}+\frac{1}{2}\nabla_{\vec
x'}\psi'\cdot\vec v'\right)=-\frac{\hbar^2}{2m'}\Delta_{\vec
x'}\psi'+\frac{i\hbar\sigma'}{2m'c}div_{\vec x'}\left\{\psi'\vec
A'\right\}+\frac{i\hbar\sigma'}{2m'c}\nabla_{\vec x'}\psi'\cdot\vec
A'\\+\frac{(\sigma')^2}{2m'c^2}\left|\vec A'\right|^2\psi'
%
%
%
+\sigma'\left(\Psi'-\frac{1}{c}\vec v'\cdot\vec
A'\right)\psi'-V'\psi'+\frac{\hbar}{2}\vec
S\cdot\left(\left(\frac{1}{2}curl_{\vec x'}\vec
v'-\frac{g'\sigma'}{m'c}curl_{\vec x'}\vec A'\right)\,\psi'\right),
\end{multline}
where $\psi\in\mathbb{C}^2$. Then in the coordinate system $(*)$ we
have the validity of Shr\"{o}dinger-Pauli equation of the same as
\er{MaxVacFull1ninshtrredPPNintrrZZ} form:
\begin{multline}\label{MaxVacFull1ninshtrhjkkredPPNintrrZZ}
i\hbar\left(\frac{\partial\psi}{\partial t}+\frac{1}{2}div_{\vec
x}\left\{\psi\vec v\right\}+\frac{1}{2}\nabla_{\vec x}\psi\cdot\vec
v\right)=-\frac{\hbar^2}{2m}\Delta_{\vec
x}\psi+\frac{i\hbar\sigma}{2mc}div_{\vec x}\left\{\psi\vec
A\right\}+\frac{i\hbar\sigma}{2mc}\nabla_{\vec x}\psi\cdot\vec
A\\+\frac{\sigma^2}{2mc^2}\left|\vec A\right|^2\psi
%
%
%
+\sigma\left(\Psi-\frac{1}{c}\vec v\cdot\vec
A\right)\psi-V\psi+\frac{\hbar}{2}\vec
S\cdot\left(\left(\frac{1}{2}curl_{\vec x}\vec
v-\frac{g\sigma}{mc}curl_{\vec x}\vec A\right)\,\psi\right),
\end{multline}
provided that
\begin{equation}\label{yuythfgfyftydtydtydtyddyyyhhddhhhredPPN111hgghjgintintrrZZ}
\begin{cases}
g'=g\\
V'=V,\\
\sigma'=\sigma,\\
m'=m,\\
\vec v'=A(t)\cdot \vec v+\frac{dA}{dt}(t)\cdot\vec x+\frac{d\vec z}{dt}(t),\\
\vec A'=A(t)\cdot \vec A,\\
\Psi'-\vec v'\cdot\vec A'=\Psi-\vec v\cdot\vec A,\\
\psi'=U(t)\cdot\psi,
\end{cases}
\end{equation}
where $U(t)\in SU(2)$ is some special unitary $2\times 2$ matrix
i.e. $U(t)\in\mathbb{C}^{2\times 2}$, $det\,U(t)=1$, $U(t)\cdot
U^*(t)=I$ where $U^*(t)$ is the Hermitian adjoint to $U(t)$ matrix:
$U^*(t):=\bar U(t)^T$ and $I$ is the identity $2\times 2$ matrix.
Moreover, $U(t)$ is characterized by:
\begin{equation}\label{gyfyfgfgfgghZZ}
U^*(t)\cdot\vec S\cdot U(t)=A(t)\cdot\vec S,
\end{equation}
that means
\begin{multline*}
\left(U^*(t)\cdot S_1\cdot U(t),U^*(t)\cdot S_2\cdot
U(t),U^*(t)\cdot S_3\cdot U(t)\right)=\\
\left(a_{11}(t)S_1+a_{12}(t)S_2+a_{13}(t)S_3\,,\,a_{21}(t)S_1+a_{22}(t)S_2+a_{23}(t)S_3\,,\,a_{31}(t)S_1+a_{32}(t)S_2+a_{33}(t)S_3\right),
\end{multline*}
where $A(t)=\left\{a_{mk}(t)\right\}_{\{1\leq m,k\leq 3\}}$.
\end{theorem}
\begin{proof}
It is well known that we have
\begin{equation}\label{gyjyfghfZZ}
\begin{cases}
S^2_1=S^2_2=S^2_3=I,\\
S_1\cdot S_2=-S_2\cdot S_1=iS_3,\quad S_2\cdot S_3=-S_3\cdot
S_2=iS_1,\quad S_3\cdot S_1=-S_1\cdot S_3=iS_2.
\end{cases}
\end{equation}
Next, it is well known that $SO(3)$ is smoothly double covered by
$SU(2)$ and the cover mapping is regular and locally one to one. In
particular, for every $t$ there exists $U(t)\in SU(2)$ such that
\er{gyfyfgfgfgghZZ} is satisfied
(note that the seconde choice is $(-U(t))$). Moreover, by Implicit
Function Theorem we deduce that if $A(t)$ is differentiable by $t$
then $U(t)$ is also differentiable by $t$. Thus if
$\psi'=U(t)\cdot\psi$ in \er{MaxVacFull1ninshtrredPPNintrrZZ}, then
by \er{MaxVacFull1ninshtrredPPNintrrZZ},
\er{yuythfgfyftydtydtydtyddyyyhhddhhhredPPN111hgghjgintintrrZZ} and
proposition \ref{yghgjtgyrtrt}
we have:
\begin{multline}\label{MaxVacFull1ninshtrredPPNintrrghghZZ}
i\hbar
U(t)\cdot\left(U^{-1}(t)\cdot\frac{dU}{dt}(t)\cdot\psi+\frac{\partial\psi}{\partial
t}+\frac{1}{2}div_{\vec x}\left\{\psi\vec
v\right\}+\frac{1}{2}\nabla_{\vec x}\psi\cdot\vec v\right)
=\\U(t)\cdot\left(-\frac{\hbar^2}{2m}\Delta_{\vec
x}\psi+\frac{i\hbar\sigma}{2mc}div_{\vec x}\left\{\psi\vec
A\right\}+\frac{i\hbar\sigma}{2mc}\nabla_{\vec x}\psi\cdot\vec
A+\frac{\sigma^2}{2mc^2}\left|\vec A\right|^2\psi
%
%
%
+\sigma\left(\Psi-\frac{1}{c}\vec v\cdot\vec
A\right)\psi-V\psi\right)\\+\frac{\hbar}{2}\vec
S\cdot\left(A(t)\cdot\left(\frac{1}{2}curl_{\vec
x}\left(A^{-1}(t)\cdot\frac{dA}{dt}(t)\cdot \vec
x\right)+\frac{1}{2}curl_{\vec x}\vec v-\frac{g\sigma}{mc}curl_{\vec
x}\vec A\right)\,U(t)\cdot\psi\right).
\end{multline}
Thus, since $U(t)$ is unitary and then $U^{-1}(t)=U^*(t)$ and
$A^{-1}(t)=A^T(t)$, by \er{MaxVacFull1ninshtrredPPNintrrghghZZ} we
have
\begin{multline}\label{MaxVacFull1ninshtrredPPNintrrghghghghZZ}
\frac{\hbar}{4}\left(4iU^*(t)\cdot\frac{dU}{dt}(t)\cdot\psi-
U^*(t)\cdot\vec S\cdot U(t)\cdot\left(A(t)\cdot\left(curl_{\vec
x}\left(A^T(t)\cdot\frac{dA}{dt}(t)\cdot \vec
x\right)\right)\psi\right)\right)
\\+i\hbar
\left(\frac{\partial\psi}{\partial t}+\frac{1}{2}div_{\vec
x}\left\{\psi\vec v\right\}+\frac{1}{2}\nabla_{\vec x}\psi\cdot\vec
v\right) =\\-\frac{\hbar^2}{2m}\Delta_{\vec
x}\psi+\frac{i\hbar\sigma}{2mc}div_{\vec x}\left\{\psi\vec
A\right\}+\frac{i\hbar\sigma}{2mc}\nabla_{\vec x}\psi\cdot\vec
A+\frac{\sigma^2}{2mc^2}\left|\vec A\right|^2\psi
%
%
%
+\sigma\left(\Psi-\frac{1}{c}\vec v\cdot\vec
A\right)\psi-V\psi\\+\frac{\hbar}{2}U^*(t)\cdot\vec S\cdot
U(t)\cdot\left(A(t)\cdot\left(\frac{1}{2}curl_{\vec x}\vec
v-\frac{g\sigma}{mc}curl_{\vec x}\vec A\right)\,\psi\right).
\end{multline}
Thus by \er{MaxVacFull1ninshtrredPPNintrrghghghghZZ} and
\er{gyfyfgfgfgghZZ} we deduce:
\begin{multline}\label{MaxVacFull1ninshtrredPPNintrrghghghghuhjihZZ}
\frac{\hbar}{4}\left(4iU^*(t)\cdot\frac{dU}{dt}(t)\cdot\psi- \vec
S\cdot \left(curl_{\vec x}\left(A^T(t)\cdot\frac{dA}{dt}(t)\cdot
\vec x\right)\right)\psi\right)
\\+i\hbar
\left(\frac{\partial\psi}{\partial t}+\frac{1}{2}div_{\vec
x}\left\{\psi\vec v\right\}+\frac{1}{2}\nabla_{\vec x}\psi\cdot\vec
v\right) =\\-\frac{\hbar^2}{2m}\Delta_{\vec
x}\psi+\frac{i\hbar\sigma}{2mc}div_{\vec x}\left\{\psi\vec
A\right\}+\frac{i\hbar\sigma}{2mc}\nabla_{\vec x}\psi\cdot\vec
A+\frac{\sigma^2}{2mc^2}\left|\vec A\right|^2\psi
%
%
%
+\sigma\left(\Psi-\frac{1}{c}\vec v\cdot\vec
A\right)\psi-V\psi\\+\frac{\hbar}{2}\vec S\cdot
\left(\left(\frac{1}{2}curl_{\vec x}\vec
v-\frac{g\sigma}{mc}curl_{\vec x}\vec A\right)\,\psi\right).
\end{multline}
On the other hand differentiating the identities $A^T(t)\cdot
A(t)=I$ and $U^*(t)\cdot U(t)=I$ by $t$ we deduce that the real
$3\times 3$ matrix $A^T(t)\cdot\frac{dA}{dt}(t)$ is antisymmetric
and the complex $2\times 2$ matrix $iU^*(t)\cdot\frac{dU}{dt}(t)$ is
Hermitian self-adjoint. Moreover, differentiating the identity
$det\,U(t)=1$ and using that $U^*(t)=U^{-1}(t)$ we deduce that the
matrix $iU^*(t)\cdot\frac{dU}{dt}(t)$ is traceless. In particular,
since $A^T(t)\cdot\frac{dA}{dt}(t)$ is antisymmetric, there exists
$\vec w(t)=(w_1(t),w_2(t),w_3(t))\in\mathbb{R}^3$ such that
\begin{equation}\label{gyfyfgfgfgghuyyuyyuZZ}
A^T(t)\cdot\frac{dA}{dt}(t)\cdot \vec x=\vec w(t)\times\vec
x\quad\quad\forall\vec x\in\mathbb{R}^3,
\end{equation}
and then,
\begin{equation}\label{gyfyfgfgfgghuyyuyyuuyhjjhZZ}
curl_{\vec x}\left(A^T(t)\cdot\frac{dA}{dt}(t)\cdot \vec
x\right)=2\vec w(t).
\end{equation}
On the other hand, since the matrix $iU^*(t)\cdot\frac{dU}{dt}(t)$
is Hermitian self-adjoint and traceless, clearly there exist $\vec
q(t)=(q_1(t),q_2(t),q_3(t))\in\mathbb{R}^3$ such that
\begin{equation}\label{gyfyfgfgfgghuyyuyyuuyhjjhgfhhjghZZ}
iU^*(t)\cdot\frac{dU}{dt}(t)=\vec q(t)\cdot\vec
S:=q_1(t)S_1+q_2(t)S_2+q_3(t)S_3.
\end{equation}
Finally differentiating the identity \er{gyfyfgfgfgghZZ} by $t$ we
deduce
\begin{equation}\label{gyfyfgfgfgghyuyyugghhgZZ}
\frac{dU^*}{dt}(t)\cdot\vec S\cdot U(t)+U^*(t)\cdot\vec S\cdot
\frac{dU}{dt}(t)=\frac{dA}{dt}(t)\cdot\vec S,
\end{equation}
and then again inserting \er{gyfyfgfgfgghZZ} into
\er{gyfyfgfgfgghyuyyugghhgZZ} and using the antisymmetry of
$A^T(t)\cdot\frac{dA}{dt}(t)$ and the Hermitian property of
$iU^*(t)\cdot\frac{dU}{dt}(t)$ we obtain
\begin{equation}\label{gyfyfgfgfgghyuyyugghhgghghhhjZZ}
-i\left(\left(A(t)\cdot\vec S\right)\cdot \left(iU^*(t)\cdot
\frac{dU}{dt}(t)\right)-\left(i\cdot
U^*(t)\cdot\frac{dU}{dt}(t)\right)\cdot \left(A(t)\cdot\vec
S\right)\right)=A(t)\cdot\left(A^T(t)\cdot\frac{dA}{dt}(t)\right)\cdot\vec
S,
\end{equation}
that implies:
\begin{equation}\label{gyfyfgfgfgghyuyyugghhgghghhhjuytytytyZZ}
-i\left(\vec S\cdot \left(iU^*(t)\cdot
\frac{dU}{dt}(t)\right)-\left(i\cdot
U^*(t)\cdot\frac{dU}{dt}(t)\right)\cdot \vec
S\right)=\left(A^T(t)\cdot\frac{dA}{dt}(t)\right)\cdot\vec S,
\end{equation}
Then inserting \er{gyfyfgfgfgghuyyuyyuZZ} and
\er{gyfyfgfgfgghuyyuyyuuyhjjhgfhhjghZZ} into
\er{gyfyfgfgfgghyuyyugghhgghghhhjuytytytyZZ} we deduce
\begin{equation}\label{gyfyfgfgfgghyuyyugghhgghghhhjuytytytytytyyuuyZZ}
-i\left(\vec S\cdot
\left(q_1(t)S_1+q_2(t)S_2+q_3(t)S_3\right)-\left(q_1(t)S_1+q_2(t)S_2+q_3(t)S_3\right)\cdot
\vec S\right)=\vec w(t)\times\vec S.
\end{equation}
Thus by \er{gyjyfghfZZ} and
\er{gyfyfgfgfgghyuyyugghhgghghhhjuytytytytytyyuuyZZ} we get
\begin{equation}\label{gyfyfgfgfgghyuyyugghhgghghhhjuytytytytytyyuuyghhgZZ}
2\vec q(t)=\vec w(t).
\end{equation}
Therefore, \er{gyfyfgfgfgghyuyyugghhgghghhhjuytytytytytyyuuyghhgZZ},
\er{gyfyfgfgfgghuyyuyyuuyhjjhgfhhjghZZ} and
\er{gyfyfgfgfgghuyyuyyuuyhjjhZZ} together imply:
\begin{equation}\label{gyfyfgfgfgghyuyyugghhgghghhhjuytytytytytyyuuyghhgjhjhZZ}
4iU^*(t)\cdot\frac{dU}{dt}(t)=\vec S\cdot \left(curl_{\vec
x}\left(A^T(t)\cdot\frac{dA}{dt}(t)\cdot \vec x\right)\right),
\end{equation}
and inserting
\er{gyfyfgfgfgghyuyyugghhgghghhhjuytytytytytyyuuyghhgjhjhZZ} into
\er{MaxVacFull1ninshtrredPPNintrrghghghghuhjihZZ} we finally
conclude \er{MaxVacFull1ninshtrhjkkredPPNintrrZZ}.
\end{proof}

Next, again consider the motion of a quantum micro-particle with
spin-half, inertial mass $m$ and the charges $\sigma$ with the known
gravitational and electromagnetical field with characteristics $\vec
v(\vec x,t)$, $\vec A(\vec x,t)$ and $\Psi(\vec x,t)$ and additional
conservative field with potential $V(\vec y,t)$, taking into the
account spin interaction. Then consider a Lagrangian density $L$
defined by
\begin{multline}\label{vhfffngghkjgghDDmmkkkZZ}
L\left(\psi,\vec x,t\right):=
\frac{i\hbar}{2}\left(\left(\frac{\partial\psi}{\partial t}+\vec
v\cdot\nabla_{\vec
x}\psi\right)\cdot\bar\psi-\psi\cdot\left(\frac{\partial\bar\psi}{\partial
t}+\vec v\cdot\nabla_{\vec
x}\bar\psi\right)\right)-\frac{\hbar^2}{2m}\nabla_{\vec
x}\psi\cdot\nabla_{\vec x}\bar\psi+V\left(\vec
x,t\right)\psi\cdot\bar\psi\\
-\frac{\hbar\sigma i}{2mc}\left(\nabla_{\vec
x}\psi\cdot\bar\psi-\psi\cdot\nabla_{\vec x}\bar\psi\right)\cdot\vec
A-\frac{\sigma^2}{2mc^2}\left|\vec A\right|^2\psi\cdot\bar\psi
-\sigma\left(\Psi-\frac{1}{c}\vec v\cdot\vec
A\right)\psi\cdot\bar\psi\\ -\frac{\hbar}{2}\left(\left(\vec
S\cdot\left(\frac{1}{2}curl_{\vec x}\vec
v-\frac{g\sigma}{mc}curl_{\vec x}\vec
A\right)\right)\cdot\psi\right)\cdot\bar\psi,
\end{multline}
where $\psi\in \mathbb{C}^2$ is a two-component wave function. Then
similarly to the proof of Theorem \ref{gjghghgghgintintrrZZ} we can
prove that $L$ is invariant under the change of inertial or
non-inertial cartesian coordinate system, given by
\er{noninchgravortbstrjgghguittu2intrrrZZ}, provided that we take
into account
\er{yuythfgfyftydtydtydtyddyyyhhddhhhredPPN111hgghjgintintrrZZ}. We
investigate stationary points of the functional
\begin{equation}\label{btfffygtgyggyDDmmkkkZZ}
J=\int_0^T\int_{\mathbb{R}^3}L\left(\psi,\vec x,t\right)d\vec x dt.
\end{equation}
Then, by \er{vhfffngghkjgghDDmmkkkZZ} we have
\begin{multline}\label{vhfffngghkjgghDDmmkkkghgZZ}
0=\frac{\delta L}{\delta (\bar\psi)}=
i\hbar\left(\frac{\partial\psi}{\partial t}+\frac{1}{2}\vec
v\cdot\nabla_{\vec x}\psi+\frac{1}{2}div_{\vec x}\left\{\psi\vec
v\right\}\right)+\frac{\hbar^2}{2m}\Delta_{\vec x}\psi+V\left(\vec
x,t\right)\psi\\
-\frac{\hbar\sigma i}{2mc}\left(\vec A\cdot\nabla_{\vec
x}\psi+div_{\vec x}\left\{\psi\vec
A\right\}\right)-\frac{\sigma^2}{2mc^2}\left|\vec A\right|^2\psi
-\sigma\left(\Psi-\frac{1}{c}\vec v\cdot\vec A\right)\psi
\\-\frac{\hbar}{2}\left(\vec S\cdot\left(\frac{1}{2}curl_{\vec x}\vec
v-\frac{g\sigma}{mc}curl_{\vec x}\vec A\right)\right)\cdot\psi,
\end{multline}
and
\begin{multline}\label{vhfffngghkjgghDDmmkkkghguhyuyZZ}
0=\frac{\delta L}{\delta (\psi)}= \bar
{(i)}\hbar\left(\frac{\partial\bar\psi}{\partial t}+\frac{1}{2}\vec
v\cdot\nabla_{\vec x}\bar\psi+\frac{1}{2}div_{\vec
x}\left\{\bar\psi\vec
v\right\}\right)+\frac{\hbar^2}{2m}\Delta_{\vec x}\bar\psi
+V\left(\vec
x,t\right)\bar\psi\\
-\frac{\hbar\sigma \bar {(i)}}{2mc}\left(\vec A\cdot\nabla_{\vec
x}\bar\psi+div_{\vec x}\left\{\bar\psi\vec
A\right\}\right)-\frac{\sigma^2}{2mc^2}\left|\vec A\right|^2\bar\psi
-\sigma\left(\Psi-\frac{1}{c}\vec v\cdot\vec A\right)\bar\psi
\\-\frac{\hbar}{2}\left(\vec S^T\cdot\left(\frac{1}{2}curl_{\vec x}\vec
v-\frac{g\sigma}{mc}curl_{\vec x}\vec A\right)\right)\cdot\bar\psi.
\end{multline}
Note that the last equality is just the complex conjugate of
equality \er{vhfffngghkjgghDDmmkkkghgZZ}. So we get that the
Euler-Lagrange equation for \er{vhfffngghkjgghDDmmkkkZZ} coincides
with the Shr\"{o}dinger-Pauli equation in the form of
\er{vhfffngghkjgghfjjghghghSYShmyuuiiuuhmhmiopoopnniukjhjkk;l;lkhjjkihjjhkkkkkjjjZZ}.

\subsection{Shr\"{o}dinger-Pauli equation for a system of $n$
spin-half micro-particles}\label{ghghggjghfgjk}
Consider the motion
of a system of $n$ spin-half quantum micro-particles with inertial
masses $m_1,\ldots,m_n$ and the charges $\sigma_1,\ldots,\sigma_n$
in the outer gravitational and electromagnetical field with
characteristics $\vec v(\vec x,t)$, $\vec A(\vec x,t)$ and
$\Psi(\vec x,t)$ and additional conservative field with potential
$V(\vec y_1,\ldots,\vec y_n,t)$, taking into account the spin
interaction. Then the system is characterized by $2^n$-component
complex wave function $\psi(\vec x_1,\ldots,\vec
x_n,t)\in\mathbb{C}^{2^n}$ where by $\mathbb{C}^{2^n}$ we denote the
tensor product of $n$ copies of the space $\mathbb{C}^2$:
\begin{equation}\label{fyfgjguyiuiHH}
\mathbb{C}^{2^n}:=\left(\mathbb{C}^2\right)\otimes_1\left(\mathbb{C}^2\right)\otimes_2\left(\mathbb{C}^2\right)
\ldots\otimes_{(n-1)}\left(\mathbb{C}^2\right)
\end{equation}
Then we built the Hermitian Hamiltonian operator as:
\begin{multline}\label{vhfffngghkjgghfjjghghghhjghjgghkghgghjhggjjkgSYSPNHH}
\hat H_0\cdot\psi=-\sum_{j=1}^{n}\frac{i\hbar}{2}div_{\vec
x_j}\left\{\psi\vec v(\vec
x_j,t)\right\}-\sum_{j=1}^{n}\frac{i\hbar}{2}\vec v(\vec
x_j,t)\cdot\nabla_{\vec x_j}\psi\\+
\sum_{j=1}^{n}\left\{\frac{1}{2m_j}\left(-i\hbar\nabla_{\vec
x_j}-\frac{\sigma_j}{c}\vec A(\vec
x_j,t)\right)\circ\left(-i\hbar\nabla_{\vec
x_j}-\frac{\sigma_j}{c}\vec A(\vec
x_j,t)\right)\right\}\cdot\psi\\+\sum_{j=1}^{n}\sigma_j\left(\Psi(\vec
x_j,t)-\frac{1}{c}\vec A(\vec x_j,t)\cdot\vec v(\vec
x_j,t)\right)\cdot\psi-V\left(\vec x_1,\ldots,\vec
x_n,t\right)\cdot\psi
\\-\sum_{j=1}^{n}\frac{g_j\sigma_j\hbar}{2m_jc}\vec
S_j\cdot\left(curl_{\vec x_j}\vec A(\vec
x_j,t)\,\psi\right)+\sum_{j=1}^{n}\frac{\hbar}{4}\vec
S_j\cdot\left(curl_{\vec x_j}\vec v(\vec x_j,t)\,\psi\right)=
-\sum_{j=1}^{n}\frac{\hbar^2}{2m_j}\Delta_{\vec
x_j}\psi\\-\sum_{j=1}^{n}\frac{i\hbar}{2}div_{\vec
x_j}\left\{\psi\vec v(\vec
x_j,t)\right\}-\sum_{j=1}^{n}\frac{i\hbar}{2}\vec v(\vec
x_j,t)\cdot\nabla_{\vec x_j}\psi+\sum_{j=1}^{n}\frac{
i\hbar\sigma_j}{2m_jc}div_{\vec x_j}\left\{\psi\vec A(\vec
x_j,t)\right\}+\sum_{j=1}^{n}\frac{ i\hbar\sigma_j}{2m_jc}\vec
A(\vec x_j,t)\cdot\nabla_{\vec
x_j}\psi\\+\sum_{j=1}^{n}\left(\sigma_j\Psi(\vec
x_j,t)-\frac{\sigma_j}{c}\vec A(\vec x_j,t)\cdot\vec v(\vec
x_j,t)+\frac{\sigma^2_j}{2m_jc^2}\left|\vec A(\vec
x_j,t)\right|^2\right)\psi-V\left(\vec x_1,\ldots,\vec
x_n,t\right)\psi
\\-\sum_{j=1}^{n}\frac{g_j\sigma_j\hbar}{2m_jc}\vec
S_j\cdot\left(curl_{\vec x_j}\vec A(\vec
x_j,t)\,\psi\right)+\sum_{j=1}^{n}\frac{\hbar}{4}\vec
S_j\cdot\left(curl_{\vec x_j}\vec v(\vec x_j,t)\,\psi\right),
\end{multline}
where $\hat H_0$ is the Hamiltonian operator, for every
$j=1,\ldots,n$ $g_j$ is a constant that depends on the type of the
particle (for electron we have $g_j=1$), and for every $j=1,\ldots,
n$ we denote
\begin{equation}\label{fyfgjguyiuiHHHHHH}
\vec S_j:=(S^j_1,S^j_2,S^j_3)\quad\quad\forall\, j=1,2,\ldots,n,
\end{equation}
where for every $k=1,2,3$ and every $j=1,2,\ldots,n$:
$S^j_k:\mathbb{C}^{2^n}\to \mathbb{C}^{2^n}$ is a linear operator on
$\mathbb{C}^{2^n}$ (i.e. it is a $2^n\times 2^n$-complex matrix)
defined by the following identities:
\begin{multline}\label{fyfgjguyiuiHHHHHHHHHHuiu}
S^1_k:=\left(S_k\right)\otimes_{1}\left(I^{2\times
2}\right)\otimes_{2}\left(I^{2\times
2}\right)\ldots\otimes_{(n-1)}\left(I^{2\times
2}\right),\quad\quad\ldots\quad
\\
S^j_k:=\left(I^{2\times 2}\right)\otimes_1\left(I^{2\times
2}\right)\otimes_2\left(I^{2\times 2}\right)
\ldots\otimes_{(j-1)}\left(S_k\right)\otimes_{j}\left(I^{2\times
2}\right)\otimes_{(j+1)}\left(I^{2\times
2}\right)\ldots\otimes_{(n-1)}\left(I^{2\times 2}\right),
\\
\quad\quad\ldots\quad\text{and}\quad\quad S^n_k:=\left(I^{2\times
2}\right)\otimes_1\left(I^{2\times
2}\right)\otimes_2\left(I^{2\times 2}\right)
\ldots\otimes_{(n-1)}\left(S_k\right),
\end{multline}
Here $S_k$ for $k=1,2,3$ are Pauli matrixes defined as:
$$S_1=\left(\begin{matrix}0&1\\1&0\end{matrix}\right),\quad
S_2=\left(\begin{matrix}0&-i\\i&0\end{matrix}\right)\quad
S_3=\left(\begin{matrix}1&0\\0&-1\end{matrix}\right)$$ and the sign
$\otimes$ in \er{fyfgjguyiuiHHHHHHHHHHuiu} means the tensor product
of the matrices, i.e. for given two linear operators
$A:\mathbb{C}^{p}\to\mathbb{C}^{p}$ and
$B:\mathbb{C}^{q}\to\mathbb{C}^{q}$
their tensor product
$A\otimes B$ is a linear operator from
$\mathbb{C}^{p}\otimes\mathbb{C}^{q}$ to
$\mathbb{C}^{p}\otimes\mathbb{C}^{q}$, defined by the identity:
\begin{equation}\label{fyfgjguyiuiHHHHHHkjjjjl}
(A\otimes B)\cdot(a\otimes b)=(A\cdot a)\otimes (B\cdot
b)\quad\quad\forall a\in\mathbb{C}^{p},\;\forall b\in\mathbb{C}^{q}.
\end{equation}
Note that, in addition to the classical term of the spin-magnetic
interaction, we added another term to the Hamiltonian, namely
$\sum_{j=1}^{n}\frac{\hbar}{4}\vec S_j\cdot\left(curl_{\vec x_j}\vec
v(\vec x_j,t)\,\psi\right)$. In the case of the Newtonian-type
gravity, this term vanishes in every non-rotating and, in
particular, in every inertial coordinate system, however it provides
the invariance of the Shr\"{o}dinger-Pauli equation, under the
change of non-inertial cartesian coordinate system, as can be seen
in the following Theorem \ref{gjghghgghgintintrrZZHH}. Thus the
corresponding Shr\"{o}dinger-Pauli equation will be the following:
\begin{multline}\label{vhfffngghkjgghfjjghghghhjghjgghkghgghjhggjjkgfgdSYSPNHH}
i\hbar\frac{\partial\psi}{\partial t}=\hat H_0\cdot\psi=
-\sum_{j=1}^{n}\frac{\hbar^2}{2m_j}\Delta_{\vec
x_j}\psi-\sum_{j=1}^{n}\frac{i\hbar}{2}div_{\vec x_j}\left\{\psi\vec
v(\vec x_j,t)\right\}-\sum_{j=1}^{n}\frac{i\hbar}{2}\vec v(\vec
x_j,t)\cdot\nabla_{\vec x_j}\psi\\+\sum_{j=1}^{n}\frac{
i\hbar\sigma_j}{2m_jc}div_{\vec x_j}\left\{\psi\vec A(\vec
x_j,t)\right\}+\sum_{j=1}^{n}\frac{ i\hbar\sigma_j}{2m_jc}\vec
A(\vec x_j,t)\cdot\nabla_{\vec
x_j}\psi\\+\sum_{j=1}^{n}\left(\sigma_j\Psi(\vec
x_j,t)-\frac{\sigma_j}{c}\vec A(\vec x_j,t)\cdot\vec v(\vec
x_j,t)+\frac{\sigma^2_j}{2m_jc^2}\left|\vec A(\vec
x_j,t)\right|^2\right)\psi-V\left(\vec x_1,\ldots,\vec
x_n,t\right)\psi\\-\sum_{j=1}^{n}\frac{g_j\sigma_j\hbar}{2m_jc}\vec
S_j\cdot\left(curl_{\vec x_j}\vec A(\vec
x_j,t)\,\psi\right)+\sum_{j=1}^{n}\frac{\hbar}{4}\vec
S_j\cdot\left(curl_{\vec x_j}\vec v(\vec x_j,t)\,\psi\right).
\end{multline}
So
%
%
%
%
%
%
\begin{multline}\label{vhfffngghkjgghfjjghghghhjghjgghkghgghjhggjjkgfgdiyfgfjkjgjgSYSPNHH}
i\hbar\left(\frac{\partial\psi}{\partial t}+\sum_{j=1}^{n}\vec
v(\vec x_j,t)\cdot\nabla_{\vec
x_j}\psi\right)+\sum_{j=1}^{n}\frac{i\hbar}{2}\left(div_{\vec
x_j}\vec v(\vec x_j,t)\right)\psi=
-\sum_{j=1}^{n}\frac{\hbar^2}{2m_j}\Delta_{\vec x_j}\psi-V\left(\vec
x_1,\ldots,\vec x_n,t\right)\psi\\+\sum_{j=1}^{n}\frac{
i\hbar\sigma_j}{2m_jc}div_{\vec x_j}\left\{\psi\vec A(\vec
x_j,t)\right\}+\sum_{j=1}^{n}\frac{ i\hbar\sigma_j}{2m_jc}\vec
A(\vec x_j,t)\cdot\nabla_{\vec
x_j}\psi\\+\sum_{j=1}^{n}\left(\sigma_j\Psi(\vec
x_j,t)-\frac{\sigma_j}{c}\vec A(\vec x_j,t)\cdot\vec v(\vec
x_j,t)+\frac{\sigma^2_j}{2m_jc^2}\left|\vec A(\vec
x_j,t)\right|^2\right)\psi\\-\sum_{j=1}^{n}\frac{g_j\sigma_j\hbar}{2m_jc}\vec
S_j\cdot\left(curl_{\vec x_j}\vec A(\vec
x_j,t)\,\psi\right)+\sum_{j=1}^{n}\frac{\hbar}{4}\vec
S_j\cdot\left(curl_{\vec x_j}\vec v(\vec x_j,t)\,\psi\right).
\end{multline}
Then in the similar way as the proof of Theorem
\ref{gjghghgghgintintrrZZ} we can prove the following more general
Theorem about the invariance of the Shr\"{o}dinger-Pauli equation
\er{vhfffngghkjgghfjjghghghhjghjgghkghgghjhggjjkgfgdiyfgfjkjgjgSYSPNHH}
under the change of inertial or non-inertial cartesian coordinate
system:
\begin{theorem}\label{gjghghgghgintintrrZZHH}
Consider that the change of some cartesian coordinate system $(*)$
to another cartesian coordinate system $(**)$ is given by
\begin{equation}\label{noninchgravortbstrjgghguittu2intrrrZZHH}
\begin{cases}
t'=t,\\
\vec x'=A(t)\cdot\vec x+\vec z(t),\\
\vec x'_k=A(t)\cdot\vec x_k+\vec z(t)\quad\quad\forall k=1,\ldots,n,
\end{cases}
\end{equation}
where $A(t)\in SO(3)$ is a rotation. Next, assume that in the
coordinate system $(**)$ we observe a validity of the
Shr\"{o}dinger-Pauli equation of the form:
\begin{multline}\label{MaxVacFull1ninshtrredPPNintrrZZHH}
i\hbar\left(\frac{\partial\psi'}{\partial t'}+\sum_{j=1}^{n}\vec
v'(\vec x'_j,t')\cdot\nabla_{\vec
x'_j}\psi'\right)+\sum_{j=1}^{n}\frac{i\hbar}{2}\left(div_{\vec
x'_j}\vec v'(\vec x'_j,t')\right)\psi'=
-\sum_{j=1}^{n}\frac{\hbar^2}{2m'_j}\Delta_{\vec
x'_j}\psi'\\-V'\left(\vec x'_1,\ldots,\vec
x'_n,t'\right)\psi'+\sum_{j=1}^{n}\frac{
i\hbar\sigma'_j}{2m'_jc}div_{\vec x'_j}\left\{\psi'\vec A'(\vec
x'_j,t')\right\}+\sum_{j=1}^{n}\frac{ i\hbar\sigma'_j}{2m'_jc}\vec
A'(\vec x'_j,t')\cdot\nabla_{\vec
x'_j}\psi'\\+\sum_{j=1}^{n}\left(\sigma'_j\Psi'(\vec
x'_j,t')-\frac{\sigma'_j}{c}\vec A'(\vec x'_j,t')\cdot\vec v'(\vec
x'_j,t')+\frac{(\sigma')^2_j}{2m'_jc^2}\left|\vec A'(\vec
x'_j,t')\right|^2\right)\psi'\\-\sum_{j=1}^{n}\frac{g'_j\sigma'_j\hbar}{2m'_jc}\vec
S_j\cdot\left(curl_{\vec x'_j}\vec A'(\vec
x'_j,t')\,\psi'\right)+\sum_{j=1}^{n}\frac{\hbar}{4}\vec
S_j\cdot\left(curl_{\vec x'_j}\vec v'(\vec x'_j,t')\,\psi'\right)
\end{multline}
where $\psi(\vec x_1,\ldots,\vec x_n,t)\in\mathbb{C}^{2^n}$ is a
$2^n$-component complex wave function defined above. Then in the
coordinate system $(*)$ we have the validity of Shr\"{o}dinger-Pauli
equation of the same as \er{MaxVacFull1ninshtrredPPNintrrZZHH} form:
\begin{multline}\label{MaxVacFull1ninshtrhjkkredPPNintrrZZHH}
i\hbar\left(\frac{\partial\psi}{\partial t}+\sum_{j=1}^{n}\vec
v(\vec x_j,t)\cdot\nabla_{\vec
x_j}\psi\right)+\sum_{j=1}^{n}\frac{i\hbar}{2}\left(div_{\vec
x_j}\vec v(\vec x_j,t)\right)\psi=
-\sum_{j=1}^{n}\frac{\hbar^2}{2m_j}\Delta_{\vec
x_j}\psi\\-V\left(\vec x_1,\ldots,\vec
x_n,t\right)\psi+\sum_{j=1}^{n}\frac{
i\hbar\sigma_j}{2m_jc}div_{\vec x_j}\left\{\psi\vec A(\vec
x_j,t)\right\}+\sum_{j=1}^{n}\frac{ i\hbar\sigma_j}{2m_jc}\vec
A(\vec x_j,t)\cdot\nabla_{\vec
x_j}\psi\\+\sum_{j=1}^{n}\left(\sigma_j\Psi(\vec
x_j,t)-\frac{\sigma_j}{c}\vec A(\vec x_j,t)\cdot\vec v(\vec
x_j,t)+\frac{\sigma^2_j}{2m_jc^2}\left|\vec A(\vec
x_j,t)\right|^2\right)\psi\\-\sum_{j=1}^{n}\frac{g_j\sigma_j\hbar}{2m_jc}\vec
S_j\cdot\left(curl_{\vec x_j}\vec A(\vec
x_j,t)\,\psi\right)+\sum_{j=1}^{n}\frac{\hbar}{4}\vec
S_j\cdot\left(curl_{\vec x_j}\vec v(\vec x_j,t)\,\psi\right),
\end{multline}
provided that we have:
\begin{equation}\label{yuythfgfyftydtydtydtyddyyyhhddhhhredPPN111hgghjgintintrrZZHH}
\begin{cases}
g'_j\,=\,g_j\\
V'\left(\vec x'_1,\ldots,\vec x'_n,t'\right)\,=\,V\left(\vec
x_1,\ldots,\vec x_n,t\right),\\
\sigma'_j\,=\,\sigma_j,\\
m'_j\,=\,m_j,\\
\vec v'(\vec x',t)\,=\,A(t)\cdot \vec v(\vec x,t)+\frac{dA}{dt}(t)\cdot\vec x+\frac{d\vec z}{dt}(t),\\
\vec A'(\vec x',t)\,=\,A(t)\cdot \vec A(\vec x,t),\\
\Psi'(\vec x',t)-\vec v'(\vec x',t)\cdot\vec A'(\vec x',t)\,=\,\Psi(\vec x,t)-\vec v(\vec x,t)\cdot\vec A(\vec x,t),\\
\psi'\left(\vec x'_1,\ldots,\vec
x'_n,t'\right)\,=\,\left(U(t)\otimes_{1}U(t)\otimes_{2}U(t)\ldots\otimes_{(n-1)}U(t)
\right)\cdot\psi\left(\vec x_1,\ldots,\vec x_n,t\right),
\end{cases}
\end{equation}
where, as before, $U(t)\in SU(2)$
is characterized by:
\begin{equation}\label{gyfyfgfgfgghZZHH}
U^*(t)\cdot\vec S\cdot U(t)=A(t)\cdot\vec S,
\end{equation}
where
$\vec S:=(S_1,S_2,S_3)$,
that means
\begin{multline*}
\left(U^*(t)\cdot S_1\cdot U(t),U^*(t)\cdot S_2\cdot
U(t),U^*(t)\cdot S_3\cdot U(t)\right)=\\
\left(a_{11}(t)S_1+a_{12}(t)S_2+a_{13}(t)S_3\,,\,a_{21}(t)S_1+a_{22}(t)S_2+a_{23}(t)S_3\,,\,a_{31}(t)S_1+a_{32}(t)S_2+a_{33}(t)S_3\right),
\end{multline*}
where $A(t)=\left\{a_{mk}(t)\right\}_{\{1\leq m,k\leq 3\}}$.
\end{theorem}

Next, consider the Lagrangian of the motion of system of $n$
spin-half quantum micro-particles with inertial masses $m_1,\ldots,
m_n$ and the charges $\sigma_1,\ldots,\sigma_n$ in the outer
gravitational and electromagnetical fields with characteristics
$\vec v(\vec x,t)$, $\vec A(\vec x,t)$ and $\Psi(\vec x,t)$ and
additional conservative field with potential $V(\vec y_1,\ldots,\vec
y_n,t)$. Then consider a Lagrangian density $L_0$ defined by
\begin{multline}\label{vhfffngghkjgghDDmmkkksmfggffgKK}
L_0\left(\psi,\vec A,\Psi,\vec v,\vec x_1,\ldots,\vec x_n,t\right):=\\
\frac{i\hbar}{2}\left(\left(\frac{\partial\psi}{\partial
t}+\sum\limits_{k=1}^{n}\vec v(\vec x_k,t)\cdot\nabla_{\vec
x_k}\psi\right)\cdot\bar\psi-\psi\cdot\left(\frac{\partial\bar\psi}{\partial
t}+\sum\limits_{k=1}^{n}\vec v(\vec x_k,t)\cdot\nabla_{\vec
x_k}\bar\psi\right)\right)-\sum\limits_{k=1}^{n}\frac{\hbar^2}{2m_k}\nabla_{\vec
x_k}\psi\cdot\nabla_{\vec x_k}\bar\psi\\+V\left(\vec x_1,\ldots,\vec
x_n,t\right)\psi\cdot\bar\psi
-\sum\limits_{k=1}^{n}\frac{\hbar\sigma_k
i}{2m_kc}\left(\nabla_{\vec
x_k}\psi\cdot\bar\psi-\psi\cdot\nabla_{\vec
x_k}\bar\psi\right)\cdot\vec A(\vec
x_k,t)-\sum\limits_{k=1}^{n}\frac{\sigma^2_k}{2m_kc^2}\left|\vec
A(\vec x_k,t)\right|^2\psi\cdot\bar\psi\\
-\sum\limits_{k=1}^{n}\sigma_k\left(\Psi(\vec x_k,t)-\frac{1}{c}\vec
v(\vec x_k,t)\cdot\vec A(\vec x_k,t)\right)\psi\cdot\bar\psi
\\-\sum\limits_{k=1}^{n}\frac{\hbar}{2}\left(\left(\vec
S_k\cdot\left(\frac{1}{2}curl_{\vec x_k}\vec v(\vec
x_k,t)-\frac{g_k\sigma_k}{m_kc}curl_{\vec x_k}\vec A(\vec
x_k,t)\right)\right)\cdot\psi\right)\cdot\bar\psi.
\end{multline}
where $\psi:=\psi(\vec x_1,\ldots\vec x_n,t)\in \mathbb{C}^{2^n}$ is
a wave function of the system.
Then, as before, we can get that $L_0$ is invariant under the change
of inertial or non-inertial cartesian coordinate systems of the form
\begin{equation*}
\begin{cases}
t'=t\\
\vec x'=A(t)\cdot\vec x+\vec z(t)\\
\vec x'_k=A(t)\cdot\vec x_k+\vec z(t)\quad\quad\forall
k=1,\ldots,n,\end{cases}
\end{equation*}
provided, we take
\er{yuythfgfyftydtydtydtyddyyyhhddhhhredPPN111hgghjgintintrrZZHH}
into account. Next we investigate stationary points of the
functional
\begin{equation}\label{btfffygtgyggyDDmmkkksmKK}
J(\psi)=\int_0^T\int_{\left(\mathbb{R}^3\right)^n}L_0\left(\psi,\vec
A,\Psi,\vec v,\vec x_1,\ldots,\vec x_n,t\right)d\vec x_1\ldots,d\vec
x_n dt.
\end{equation}
Then,
\begin{multline}\label{vhfffngghkjgghDDmmkkkghgsmKK}
0=\frac{\delta L_0}{\delta (\bar\psi)}=
i\hbar\left(\frac{\partial\psi}{\partial
t}+\sum\limits_{k=1}^{n}\frac{1}{2}\vec v(\vec
x_k,t)\cdot\nabla_{\vec
x_k}\psi+\sum\limits_{k=1}^{n}\frac{1}{2}div_{\vec
x_k}\left\{\psi\vec v(\vec
x_k,t)\right\}\right)+\sum\limits_{k=1}^{n}\frac{\hbar^2}{2m_k}\Delta_{\vec
x_k}\psi\\+V\left(\vec x_1,\ldots,\vec x_n,t\right)\psi
-\frac{\hbar\sigma_k i}{2m_kc}\left(\sum\limits_{k=1}^{n}\vec A(\vec
x_k,t)\cdot\nabla_{\vec x_k}\psi+\sum\limits_{k=1}^{n}div_{\vec
x_k}\left\{\psi\vec A(\vec
x_k,t)\right\}\right)-\sum\limits_{k=1}^{n}\frac{\sigma^2_k}{2m_kc^2}\left|\vec
A(\vec x_k,t)\right|^2\psi
\\-\sum\limits_{k=1}^{n}\sigma_k\left(\Psi(\vec x_k,t)-\frac{1}{c}\vec
v(\vec x_k,t)\cdot\vec A(\vec x_k,t)\right)\psi
\\-\sum\limits_{k=1}^{n}\frac{\hbar}{2}\left(\vec
S_k\cdot\left(\frac{1}{2}curl_{\vec x_k}\vec v(\vec
x_k,t)-\frac{g_k\sigma_k}{m_kc}curl_{\vec x_k}\vec A(\vec
x_k,t)\right)\right)\cdot\psi,
\end{multline}
and
\begin{multline}\label{vhfffngghkjgghDDmmkkkghguhyuysmKK}
0=\frac{\delta L_0}{\delta (\psi)}= \bar
{(i)}\hbar\left(\frac{\partial\bar\psi}{\partial
t}+\sum\limits_{k=1}^{n}\frac{1}{2}\vec v(\vec
x_k,t)\cdot\nabla_{\vec
x_k}\bar\psi+\sum\limits_{k=1}^{n}\frac{1}{2}div_{\vec
x_k}\left\{\bar\psi\vec v(\vec
x_k,t)\right\}\right)+\sum\limits_{k=1}^{n}\frac{\hbar^2}{2m_k}\Delta_{\vec
x_k}\bar\psi \\+V\left(\vec x_1,\ldots,\vec x_n,t\right)\bar\psi
-\sum\limits_{k=1}^{n}\frac{\hbar\sigma_k \bar
{(i)}}{2m_kc}\left(\vec A(\vec x_k,t)\cdot\nabla_{\vec
x_k}\bar\psi+div_{\vec x_k}\left\{\bar\psi\vec A(\vec
x_k,t)\right\}\right)-\sum\limits_{k=1}^{n}\frac{\sigma^2_k}{2m_kc^2}\left|\vec
A(\vec x_k,t)\right|^2\bar\psi
\\-\sum\limits_{k=1}^{n}\sigma_k\left(\Psi(\vec x_k,t)-\frac{1}{c}\vec v(\vec x_k,t)\cdot\vec
A(\vec x_k,t)\right)\bar\psi
\\-\sum\limits_{k=1}^{n}\frac{\hbar}{2}\left(\vec
S^T_k\cdot\left(\frac{1}{2}curl_{\vec x_k}\vec v(\vec
x_k,t)-\frac{g_k\sigma_k}{m_kc}curl_{\vec x_k}\vec A(\vec
x_k,t)\right)\right)\cdot\bar\psi,
\end{multline}
Equation \er{vhfffngghkjgghDDmmkkkghguhyuysmKK} is just a complex
conjugate of equation \er{vhfffngghkjgghDDmmkkkghgsmKK}. Thus the
Euler-Lagrange for \er{btfffygtgyggyDDmmkkksmKK} coincides with the
Shr\"{o}dinger-Pauli equation
\er{vhfffngghkjgghfjjghghghhjghjgghkghgghjhggjjkgfgdiyfgfjkjgjgSYSPNHH}.

Finally, assume that the first and the second particles have the
same mass $m_1=m_2$ and the same charge $\sigma_1=\sigma_2$ and
moreover, assume that we have
$$V\left(\vec x_1,\vec x_2,\vec x_3,\ldots,\vec
x_n,t\right)=V\left(\vec x_2,\vec x_1,\vec x_3,\ldots,\vec
x_n,t\right)\quad\quad\forall\,\vec x_1,\vec x_2,\vec
x_3,\ldots,\vec x_n\in\mathbb{R}^3,\quad\forall t.$$ In this case it
can be easily deduced that if $\psi(\vec x_1,\vec x_2,\vec
x_3,\ldots,\vec x_n,t)\in\mathbb{C}^{2^n}$ is a solution of
\er{vhfffngghkjgghfjjghghghhjghjgghkghgghjhggjjkgfgdiyfgfjkjgjgSYSPNHH},
then $A_{1,2}\cdot\psi(\vec x_2,\vec x_1,\vec x_3,\ldots,\vec
x_n,t)$ is also a solution of
\er{vhfffngghkjgghfjjghghghhjghjgghkghgghjhggjjkgfgdiyfgfjkjgjgSYSPNHH},
where by $A_{1,2}:\mathbb{C}^{2^n}\to \mathbb{C}^{2^n}$ we denote
the linear operator (matrix) defined as:
\begin{equation}\label{fyfgjguyiuiHHgjgjgghhHH}
A_{1,2}\cdot\left(a_1\otimes a_2\otimes a_3\otimes\ldots\otimes
a_n\right)=\left(a_2\otimes a_1\otimes a_3\otimes\ldots\otimes
a_n\right)\quad\quad\forall\, a_1,\ldots a_n\in\mathbb{C}^{2}.
\end{equation}
Therefore, if $\psi(\vec x_1,\vec x_2,\vec x_3,\ldots,\vec
x_n,t)\in\mathbb{C}^{2^n}$ is a solution of
\er{vhfffngghkjgghfjjghghghhjghjgghkghgghjhggjjkgfgdiyfgfjkjgjgSYSPNHH}
then for every $t\geq 0$ we will have
\begin{equation}\label{fyfgjguyiuiHHgjgjgghbmgjgfhfhhHH}
A_{1,2}\cdot\psi(\vec x_2,\vec x_1,\vec x_3,\ldots,\vec
x_n,t)=-\psi(\vec x_1,\vec x_2,\vec x_3,\ldots,\vec
x_n,t)\quad\quad\forall\,\vec x_1,\vec x_2,\vec x_3,\ldots,\vec
x_n\in\mathbb{R}^3,
\end{equation}
provided that \er{fyfgjguyiuiHHgjgjgghbmgjgfhfhhHH} holds for the
initial instant of time $t=0$. So we have a consistency with the
Pauli Exclusion Principle for two or more identical fermions.

\subsection{Quantum Liouville's equation for a finite system of spin-half
particles}\label{hilghjhghfdgdg11} Consider the
\underline{statistical} description of the motion of a system of $n$
spin-half micro-particles with inertial masses $m_1,\ldots,m_n$ and
the charges $\sigma_1,\ldots,\sigma_n$ in the outer gravitational
and electromagnetical field with characteristics $\vec v(\vec x,t)$,
$\vec A(\vec x,t)$ and $\Psi(\vec x,t)$ and additional conservative
field with potential $V(\vec y_1,\ldots,\vec y_n,t)$, taking into
account the spin interaction. Then, as before, the Quantum
Liouville's equation for this system of particles has the following
form:
\begin{multline}\label{vhfffngghkjgghfjjghghghhjghjgghkghggkghghjghSYSPNjj11}
i\hbar\frac{\partial\xi}{\partial t}\left(\vec x_1,\ldots,\vec
x_n,\vec y_1,\ldots,\vec y_n,t\right)=\\ \left(\hat H_0\left(\vec
x_1,\ldots,\vec x_n,t\right)\otimes I\right)\cdot\xi\left(\vec
x_1,\ldots,\vec x_n,\vec y_1,\ldots,\vec
y_n,t\right)\\-\left(I\otimes\hat H^*_0\left(\vec y_1,\ldots,\vec
y_n,t\right)\right)\cdot\xi\left(\vec x_1,\ldots,\vec x_n,\vec
y_1,\ldots,\vec y_n,t\right),
\end{multline}
where $\xi\left(\vec x_1,\ldots,\vec x_n,\vec y_1,\ldots,\vec
y_n,t\right)\in\mathbb{C}^{2^n}\otimes \mathbb{C}^{2^n}$ is a
density-matrix function, $\hat H_0\left(\vec x_1,\ldots,\vec
x_n,t\right)$ is the Hamiltonian operator acting on the variables
$\left(\vec x_1,\ldots,\vec x_n\right)$, $\hat H^*_0\left(\vec
y_1,\ldots,\vec y_n,t\right)$ is the \underline{complex} conjugate
(not the Hermitian adjoint)  to the Hamiltonian operator acting on
the variables $\left(\vec y_1,\ldots,\vec y_n\right)$ and $\hat
H_0\left(\vec x_1,\ldots,\vec x_n,t\right)\otimes I$ and
$I\otimes\hat H^*_0\left(\vec y_1,\ldots,\vec y_n,t\right)$ are
linear operators acting on functions\\ $\xi\left(\vec
x_1,\ldots,\vec x_n,\vec y_1,\ldots,\vec
y_n,t\right)\in\mathbb{C}^{2^n}\otimes \mathbb{C}^{2^n}$, defined by
\begin{multline}\label{fyfgjguyiuiHHgjgjgghbmgjgfhfhhHHjhhhjhhkjkjk}
\left(\hat H_0\left(\vec x_1,\ldots,\vec x_n,t\right)\otimes
I\right)\cdot\left(\psi_1\left(\vec x_1,\ldots,\vec
x_n,t\right)\otimes \psi_2\left(\vec y_1,\ldots,\vec
y_n,t\right)\right)=\\ \left(\hat H_0\left(\vec x_1,\ldots,\vec
x_n,t\right)\cdot\psi_1\left(\vec x_1,\ldots,\vec
x_n,t\right)\right)\otimes
\psi_2\left(\vec y_1,\ldots,\vec y_n,t\right) \quad\text{and}\quad\\
\left(I\otimes\hat H^*_0\left(\vec y_1,\ldots,\vec
y_n,t\right)\right)\cdot\left(\psi_1\left(\vec x_1,\ldots,\vec
x_n,t\right)\otimes \psi_2\left(\vec y_1,\ldots,\vec
y_n,t\right)\right)=\\ \psi_1\left(\vec x_1,\ldots,\vec
x_n,t\right)\otimes \left(\hat H^*_0\left(\vec y_1,\ldots,\vec
y_n,t\right)\cdot\psi_2\left(\vec y_1,\ldots,\vec
y_n,t\right)\right)\quad\quad\forall\psi_1,\psi_2\in\mathbb{C}^{2^n}.
\end{multline}
Since by \er{vhfffngghkjgghfjjghghghSYSPN} the Hamiltonian operator
$\hat H_0$ has the forms:
\begin{multline}\label{vhfffngghkjgghfjjghghghhjghjgghkghgghjhggjjkgSYSPNjj11}
\hat H_0\left(\vec x_1,\ldots,\vec x_n,t\right)\cdot\psi\left(\vec
x_1,\ldots,\vec x_n,t\right)=
-\sum_{j=1}^{n}\frac{\hbar^2}{2m_j}\Delta_{\vec
x_j}\psi-\sum_{j=1}^{n}\frac{i\hbar}{2}div_{\vec x_j}\left\{\psi\vec
v(\vec x_j,t)\right\}\\-\sum_{j=1}^{n}\frac{i\hbar}{2}\vec v(\vec
x_j,t)\cdot\nabla_{\vec x_j}\psi+\sum_{j=1}^{n}\frac{
i\hbar\sigma_j}{2m_jc}div_{\vec x_j}\left\{\psi\vec A(\vec
x_j,t)\right\}+\sum_{j=1}^{n}\frac{ i\hbar\sigma_j}{2m_jc}\vec
A(\vec x_j,t)\cdot\nabla_{\vec
x_j}\psi\\+\sum_{j=1}^{n}\left(\sigma_j\Psi(\vec
x_j,t)-\frac{\sigma_j}{c}\vec A(\vec x_j,t)\cdot\vec v(\vec
x_j,t)+\frac{\sigma^2_j}{2m_jc^2}\left|\vec A(\vec
x_j,t)\right|^2\right)\psi-V\left(\vec x_1,\ldots,\vec
x_n,t\right)\psi
\\-\sum_{j=1}^{n}\frac{g_j\sigma_j\hbar}{2m_jc}\vec
S_j\cdot\left(curl_{\vec x_j}\vec A(\vec
x_j,t)\,\psi\right)+\sum_{j=1}^{n}\frac{\hbar}{4}\vec
S_j\cdot\left(curl_{\vec x_j}\vec v(\vec x_j,t)\,\psi\right),
\end{multline}
and consistently with
\er{vhfffngghkjgghfjjghghghhjghjgghkghgghjhggjjkgSYSPNjj11} we write
\begin{multline}\label{vhfffngghkjgghfjjghghghhjghjgghkghgghjhggjjkgSYSPNjjklkl11}
\hat H^*_0\left(\vec y_1,\ldots,\vec y_n,t\right)\cdot\psi\left(\vec
y_1,\ldots,\vec y_n,t\right)=
-\sum_{j=1}^{n}\frac{\hbar^2}{2m_j}\Delta_{\vec
y_j}\psi+\sum_{j=1}^{n}\frac{i\hbar}{2}div_{\vec y_j}\left\{\psi\vec
v(\vec y_j,t)\right\}\\+\sum_{j=1}^{n}\frac{i\hbar}{2}\vec v(\vec
y_j,t)\cdot\nabla_{\vec y_j}\psi-\sum_{j=1}^{n}\frac{
i\hbar\sigma_j}{2m_jc}div_{\vec y_j}\left\{\psi\vec A(\vec
y_j,t)\right\}-\sum_{j=1}^{n}\frac{ i\hbar\sigma_j}{2m_jc}\vec
A(\vec y_j,t)\cdot\nabla_{\vec
y_j}\psi\\+\sum_{j=1}^{n}\left(\sigma_j\Psi(\vec
y_j,t)-\frac{\sigma_j}{c}\vec A(\vec y_j,t)\cdot\vec v(\vec
y_j,t)+\frac{\sigma^2_j}{2m_jc^2}\left|\vec A(\vec
y_j,t)\right|^2\right)\psi-V\left(\vec y_1,\ldots,\vec
y_n,t\right)\psi
\\-\sum_{j=1}^{n}\frac{g_j\sigma_j\hbar}{2m_jc}\vec
{\bar S}_j\cdot\left(curl_{\vec y_j}\vec A(\vec
y_j,t)\,\psi\right)+\sum_{j=1}^{n}\frac{\hbar}{4}\vec {\bar
S}_j\cdot\left(curl_{\vec y_j}\vec v(\vec y_j,t)\,\psi\right),
\end{multline}
where $\vec {\bar S}_j$ is the the \underline{complex} conjugate
(not the Hermitian adjoint) of the operator $\vec S_j$. Thus we
rewrite the corresponding Quantum Liouville's equation
\er{vhfffngghkjgghfjjghghghhjghjgghkghggkghghjghSYSPNjj11} as:
%
%
%
%
%
%
\begin{multline}\label{vhfffngghkjgghfjjghghghhjghjgghkghgghjhggjjkgfgdiyfgfjkjgjgSYSPNjj11}
i\hbar\left(\frac{\partial\xi}{\partial t}+\sum_{j=1}^{n}\vec v(\vec
x_j,t)\cdot\nabla_{\vec x_j}\xi+\sum_{j=1}^{n}\vec v(\vec
y_j,t)\cdot\nabla_{\vec
y_j}\xi\right)+\sum_{j=1}^{n}\frac{i\hbar}{2}\left(div_{\vec
x_j}\vec v(\vec x_j,t)+div_{\vec y_j}\vec v(\vec y_j,t)\right)\xi=\\
\sum_{j=1}^{n}\frac{\hbar}{4}(\vec S_j\otimes
I)\cdot\left(curl_{\vec x_j}\vec v(\vec
x_j,t)\,\xi\right)-\sum_{j=1}^{n}\frac{\hbar}{4} (I\otimes \vec{\bar
S}_j)\cdot\left(curl_{\vec y_j}\vec v(\vec y_j,t)\,\xi\right)\\
-\sum_{j=1}^{n}\frac{\hbar^2}{2m_j}\left(\Delta_{\vec
x_j}\xi-\Delta_{\vec y_j}\xi\right)-\left(V\left(\vec
x_1,\ldots,\vec x_n,t\right)-V\left(\vec y_1,\ldots,\vec
y_n,t\right)\right)\xi\\+\sum_{j=1}^{n}\frac{
i\hbar\sigma_j}{2m_jc}\left(div_{\vec x_j}\left\{\xi\vec A(\vec
x_j,t)\right\}+div_{\vec y_j}\left\{\xi\vec A(\vec
y_j,t)\right\}\right)+\sum_{j=1}^{n}\frac{
i\hbar\sigma_j}{2m_jc}\left(\vec A(\vec x_j,t)\cdot\nabla_{\vec
x_j}\xi+\vec A(\vec y_j,t)\cdot\nabla_{\vec
y_j}\xi\right)\\+\sum_{j=1}^{n}\left(\sigma_j\Psi(\vec
x_j,t)-\frac{\sigma_j}{c}\vec A(\vec x_j,t)\cdot\vec v(\vec
x_j,t)+\frac{\sigma^2_j}{2m_jc^2}\left|\vec A(\vec
x_j,t)\right|^2\right)\xi\\- \sum_{j=1}^{n}\left(\sigma_j\Psi(\vec
y_j,t)-\frac{\sigma_j}{c}\vec A(\vec y_j,t)\cdot\vec v(\vec
y_j,t)+\frac{\sigma^2_j}{2m_jc^2}\left|\vec A(\vec
y_j,t)\right|^2\right)\xi\\-\sum_{j=1}^{n}\frac{g_j\sigma_j\hbar}{2m_jc}(\vec
S_j\otimes I)\cdot\left(curl_{\vec x_j}\vec A(\vec
x_j,t)\,\xi\right)
+\sum_{j=1}^{n}\frac{g_j\sigma_j\hbar}{2m_jc}(I\otimes \vec{\bar
S}_j)\cdot\left(curl_{\vec y_j}\vec A(\vec y_j,t)\,\xi\right),
\end{multline}
where, as before,
\begin{equation}\label{fyfgjguyiuiHHHHHHkjjjjljhjhhh}
(\vec S_j\otimes I)\cdot(a\otimes b)=(\vec S_j\cdot a)\otimes
b\quad\text{and}\quad(I\otimes\vec {\bar S}_j)\cdot(a\otimes
b)=a\otimes (\vec {\bar S}_j\cdot b)\quad\forall
a\in\mathbb{C}^{2^n},\;\forall b\in\mathbb{C}^{2^n}.
\end{equation}
Next consider a change of some non-inertial cartesian coordinate
system $(*)$ to another cartesian coordinate system $(**)$ of the
form \er{noninchgravortbstrjgghguittu2}:
\begin{equation}\label{noninchgravortbstrghgggSYSPNjj11}
\begin{cases}
\vec x'=A(t)\cdot\vec x+\vec z(t),\\
t'=t,
\end{cases}
\end{equation}
where $A(t)\in SO(3)$ is a rotation. Then, as before in Theorem
\ref{gjghghgghgintintrrZZHH}, we deduce that the Quantum Liouville's
equation equation of the form
\er{vhfffngghkjgghfjjghghghhjghjgghkghgghjhggjjkgfgdiyfgfjkjgjgSYSPNjj11}
is invariant under the change of non-inertial cartesian coordinate
system, provided we have
\begin{equation}\label{vyfgjhgjhvhgghSYSPNjj11}
\begin{cases}
\vec x'_j=A(t)\cdot\vec x_j+\vec z(t)\quad\forall j=1,\ldots,n\\
\vec y'_j=A(t)\cdot\vec y_j+\vec z(t)\quad\forall j=1,\ldots,n\\
\xi'\left(\vec x'_1,\ldots,\vec x'_n,\vec y'_1,\ldots,\vec
y'_n,t'\right)=\\
\left(\left(U(t)\otimes_{1}U(t)\ldots\otimes_{(n-1)}U(t)
\right)\otimes\left(\bar U(t)\otimes_{1}\bar
U(t)\ldots\otimes_{(n-1)}\bar U(t) \right)\right)\cdot\xi\left(\vec
x_1,\ldots,\vec x_n,\vec y_1,\ldots,\vec y_n,t\right)
\\
V'\left(\vec y'_1,\ldots,\vec y'_n,t'\right)=V\left(\vec
y_1,\ldots,\vec y_n,t\right)
\\
\vec A'=A(t)\cdot\vec A
\\
\Psi'-\frac{1}{c}\vec A'\cdot\vec v'=\Psi-\frac{1}{c}\vec A\cdot\vec
v,
\end{cases}
\end{equation}
where, as before, $U(t)\in SU(2)$
is characterized by:
\begin{equation}\label{gyfyfgfgfgghZZHHhugg}
U^*(t)\cdot\vec S\cdot U(t)=A(t)\cdot\vec S,
\end{equation} and $\bar U(t)$ is the
\underline{complex} conjugate (not the Hermitian adjoint) of the
matrix $U(t)$.

Next assume that $\hat A\left(\vec x_1,\ldots,\vec x_n,t\right)$ is
some Hermitian operator acting on the functions\\ $\psi\left(\vec
x_1,\ldots,\vec x_n,t\right)\in \mathbb{C}^{2^n}$ with respect to
spatial variables $\left(\vec x_1,\ldots,\vec x_n\right)$. Then the
average $\tilde A(t)$ of $\hat A$ on the density matrix $\xi$ is
defined by:
\begin{equation}\label{gyfyfgfgfgghZZHHhuggkjhjjgg}
\tilde A(t)=\frac{\int trace\left(\left(\hat A\left(\vec
x_1,\ldots,\vec x_n,t\right)\otimes
I\right)\cdot\xi\right)\left(\vec z_1,\ldots,\vec z_n,\vec
z_1,\ldots,\vec z_n,t\right)d\vec z_1\ldots d\vec z_n}{\int
trace\left(\xi\right)\left(\vec z_1,\ldots,\vec z_n,\vec
z_1,\ldots,\vec z_n,t\right)d\vec z_1\ldots d\vec z_n},
\end{equation}
where for given $R\in \mathbb{C}^{p}\otimes \mathbb{C}^{p}$
$trace(R)$ is a linear functional from $\mathbb{C}^{p}\otimes
\mathbb{C}^{p}$ to $\mathbb{C}$ defined by
\begin{equation}\label{fyfgjguyiuiHHHHHHkjjjjljhjhhhbbbgghg}
trace(a\otimes b)=\sum_{k=1}^{p}a_kb_k\quad\forall a=(a_1,\ldots
a_p)\in\mathbb{C}^{p},\;\forall b=(b_1,\ldots b_p)\in\mathbb{C}^{p}.
\end{equation}
On the other hand, using the fact that $\hat A$ is Hermitian, it can
be easily deduced that in addition to
\er{gyfyfgfgfgghZZHHhuggkjhjjgg} we have the following identity:
\begin{multline}\label{gyfyfgfgfgghZZHHhuggkjhjjggjhhh}
\int trace\left(\left(\hat A\left(\vec x_1,\ldots,\vec
x_n,t\right)\otimes I\right)\cdot\xi\right)\left(\vec
z_1,\ldots,\vec z_n,\vec z_1,\ldots,\vec z_n,t\right)d\vec z_1\ldots
d\vec z_n=\\ \int trace\left(\left(I\otimes\hat A^*\left(\vec
y_1,\ldots,\vec y_n,t\right)\right)\cdot\xi\right)\left(\vec
z_1,\ldots,\vec z_n,\vec z_1,\ldots,\vec z_n,t\right)d\vec z_1\ldots
d\vec z_n,
\end{multline}
where $\hat A^*\left(\vec y_1,\ldots,\vec y_n,t\right)$ is the
\underline{complex} conjugate (not the Hermitian adjoint)  to the
operator $\hat A$ acting on the variables $\left(\vec
y_1,\ldots,\vec y_n\right)$ and so,
\begin{equation}\label{gyfyfgfgfgghZZHHhuggkjhjjggjhhh99889}
\tilde A(t)=\frac{\int trace\left(\left(I\otimes\hat A^*\left(\vec
y_1,\ldots,\vec y_n,t\right)\right)\cdot\xi\right)\left(\vec
z_1,\ldots,\vec z_n,\vec z_1,\ldots,\vec z_n,t\right)d\vec z_1\ldots
d\vec z_n}{\int trace\left(\xi\right)\left(\vec z_1,\ldots,\vec
z_n,\vec z_1,\ldots,\vec z_n,t\right)d\vec z_1\ldots d\vec z_n}.
\end{equation}
In particular, by inserting \er{gyfyfgfgfgghZZHHhuggkjhjjggjhhh} in
the particular case $\hat A=\hat H_0$ into
\er{vhfffngghkjgghfjjghghghhjghjgghkghggkghghjghSYSPNjj11} we
deduce:
\begin{equation}\label{gyfyfgfgfgghZZHHhuggkjhjjggjhhhguggh}
i\hbar\frac{\partial}{\partial t}\int
trace\left(\xi\right)\left(\vec z_1,\ldots,\vec z_n,\vec
z_1,\ldots,\vec z_n,t\right)d\vec z_1\ldots d\vec z_n=0,
\end{equation}
and so,
\begin{equation}\label{gyfyfgfgfgghZZHHhuggkjhjjggjhhhjhgg}
\int trace\left(\xi\right)\left(\vec z_1,\ldots,\vec z_n,\vec
z_1,\ldots,\vec z_n,t\right)d\vec z_1\ldots d\vec z_n=\int
trace\left(\xi\right)\left(\vec z_1,\ldots,\vec z_n,\vec
z_1,\ldots,\vec z_n,0\right)d\vec z_1\ldots d\vec z_n.
\end{equation}

Next, by \er{vhfffngghkjgghfjjghghghhjghjgghkghggkghghjghSYSPNjj11}
we have
\begin{multline}\label{vhfffngghkjgghfjjghghghhjghjgghkghggkghghjghSYSPNjjjhhhjj}
i\hbar\frac{\partial\xi}{\partial t}\left(\vec x_1,\ldots,\vec
x_n,\vec y_1,\ldots,\vec y_n,t\right)=\\ \left(\hat H_0\left(\vec
x_1,\ldots,\vec x_n,t\right)\otimes I\right)\cdot\xi\left(\vec
x_1,\ldots,\vec x_n,\vec y_1,\ldots,\vec
y_n,t\right)\\-\left(I\otimes\hat H^*_0\left(\vec y_1,\ldots,\vec
y_n,t\right)\right)\cdot\xi\left(\vec x_1,\ldots,\vec x_n,\vec
y_1,\ldots,\vec y_n,t\right).
\end{multline}
Then, denote
\begin{equation}\label{gyfyfgfgfgghZZHHhuggkjhjjggjhhhjhgg22jjjkjkjj}
\xi_1\left(\vec x_1,\ldots,\vec x_n,\vec y_1,\ldots,\vec
y_n,t\right):=M\cdot\bar\xi\left(\vec y_1,\ldots,\vec y_n,\vec
x_1,\ldots,\vec x_n,t\right),
\end{equation}
where $M$ is a linear operator (matrix), acting on vectors in
$\mathbb{C}^{2^n}\otimes \mathbb{C}^{2^n}$ and satisfying:
\begin{equation}\label{fyfgjguyiuiHHHHHHkjjjjljhjhhhjhhjjh}
M\cdot(a\otimes b)=b\otimes a\quad\forall
a\in\mathbb{C}^{2^n},\;\forall b\in\mathbb{C}^{2^n},
\end{equation}
and consider
\begin{equation}\label{gyfyfgfgfgghZZHHhuggkjhjjggjhhhjhgg22jjjkjkkljjjkkjjj}
\zeta\left(\vec x_1,\ldots,\vec x_n,\vec y_1,\ldots,\vec
y_n,t\right):=\xi\left(\vec x_1,\ldots,\vec x_n,\vec y_1,\ldots,\vec
y_n,t\right)-\xi_1\left(\vec x_1,\ldots,\vec x_n,\vec
y_1,\ldots,\vec y_n,t\right).
\end{equation}
Thus,
\begin{multline}\label{vhfffngghkjgghfjjghghghhjghjgghkghggkghghjghSYSPNjjjhhhgggjj}
-i\hbar M\cdot \frac{\partial\bar\xi}{\partial t}\left(\vec
y_1,\ldots,\vec y_n,\vec x_1,\ldots,\vec x_n,t\right)=\\
\left(I\otimes\hat H^*_0\left(\vec y_1,\ldots,\vec
y_n,t\right)\right)\cdot\left(M\cdot\bar\xi\left(\vec
y_1,\ldots,\vec y_n,\vec x_1,\ldots,\vec
x_n,t\right)\right)\\-\left(\hat H_0\left(\vec x_1,\ldots,\vec
x_n,t\right)\otimes I\right)\cdot\left(M\cdot\bar\xi\left(\vec
y_1,\ldots,\vec y_n,\vec x_1,\ldots,\vec x_n,t\right)\right),
\end{multline}
and then
\begin{multline}\label{vhfffngghkjgghfjjghghghhjghjgghkghggkghghjghSYSPNjjjhhhhjhjjhjkhjhhjjtt}
i\hbar\frac{\partial\xi_1}{\partial t}\left(\vec x_1,\ldots,\vec
x_n,\vec y_1,\ldots,\vec y_n,t\right)=\\ \left(\hat H_0\left(\vec
x_1,\ldots,\vec x_n,t\right)\otimes I\right)\cdot\xi_1\left(\vec
x_1,\ldots,\vec x_n,\vec y_1,\ldots,\vec
y_n,t\right)\\-\left(I\otimes\hat H^*_0\left(\vec y_1,\ldots,\vec
y_n,t\right)\right)\cdot\xi_1\left(\vec x_1,\ldots,\vec x_n,\vec
y_1,\ldots,\vec y_n,t\right).
\end{multline}
Therefore,
\begin{multline}\label{vhfffngghkjgghfjjghghghhjghjgghkghggkghghjghSYSPNjjjhhhhjhjjhjkhjhhjj}
i\hbar\frac{\partial\zeta}{\partial t}\left(\vec x_1,\ldots,\vec
x_n,\vec y_1,\ldots,\vec y_n,t\right)=\\ \left(\hat H_0\left(\vec
x_1,\ldots,\vec x_n,t\right)\otimes I\right)\cdot\zeta\left(\vec
x_1,\ldots,\vec x_n,\vec y_1,\ldots,\vec
y_n,t\right)\\-\left(I\otimes\hat H^*_0\left(\vec y_1,\ldots,\vec
y_n,t\right)\right)\cdot\zeta\left(\vec x_1,\ldots,\vec x_n,\vec
y_1,\ldots,\vec y_n,t\right).
\end{multline}
Thus if $\zeta$ satisfies initial condition $\zeta\left(\vec
x_1,\ldots,\vec x_n,\vec y_1,\ldots,\vec y_n,0\right)=0$ then
$\zeta\left(\vec x_1,\ldots,\vec x_n,\vec y_1,\ldots,\vec
y_n,t\right)=0$ for every $t>0$. So
\begin{multline}\label{gyfyfgfgfgghZZHHhuggkjhjjggjhhhjhgg22jjjkjkkljjjkkjopiiiojj}
\xi\left(\vec x_1,\ldots,\vec x_n,\vec y_1,\ldots,\vec
y_n,0\right)=M\cdot\bar\xi\left(\vec y_1,\ldots,\vec y_n,\vec
x_1,\ldots,\vec x_n,0\right)\quad\text{implies}\quad\\
\xi\left(\vec x_1,\ldots,\vec x_n,\vec y_1,\ldots,\vec
y_n,t\right)=M\cdot\bar\xi\left(\vec y_1,\ldots,\vec y_n,\vec
x_1,\ldots,\vec x_n,t\right)\quad\forall\, t\geq 0.
\end{multline}

Finally, assume that the first and the second particles have the
same mass $m_1=m_2$ and the same charge $\sigma_1=\sigma_2$ and
moreover, assume that we have
$$V\left(\vec x_1,\vec x_2,\vec x_3,\ldots,\vec
x_n,t\right)=V\left(\vec x_2,\vec x_1,\vec x_3,\ldots,\vec
x_n,t\right)\quad\quad\forall\,\vec x_1,\vec x_2,\vec
x_3,\ldots,\vec x_n\in\mathbb{R}^3,\quad\forall t.$$ In this case it
can be easily deduced that if  $\xi\left(\vec x_1,\vec x_2,\vec
x_3,\ldots,\vec x_n,\vec y_1,\vec y_2,\vec y_3,\ldots,\vec
y_n,t\right)\in\mathbb{C}^{2^n}\otimes \mathbb{C}^{2^n}$ is a
solution of
\er{vhfffngghkjgghfjjghghghhjghjgghkghggkghghjghSYSPNjjjhhhjj}, then
$B_{1,2}\cdot\xi(\vec x_2,\vec x_1,\vec x_3,\ldots,\vec x_n,\vec
y_2,\vec y_1,\vec y_3,\ldots,\vec y_n,t)$
is also a solution of
\er{vhfffngghkjgghfjjghghghhjghjgghkghggkghghjghSYSPNjjjhhhjj},
where by $B_{1,2}:\mathbb{C}^{2^n}\otimes \mathbb{C}^{2^n}\to
\mathbb{C}^{2^n}\otimes \mathbb{C}^{2^n}$
we denote the linear operator (matrix) defined as:
\begin{multline}\label{fyfgjguyiuiHHgjgjgghhHHjjhjh}
B_{1,2}\cdot\left(\left(a_1\otimes a_2\otimes
a_3\otimes\ldots\otimes a_n\right)\otimes \left(b_1\otimes
b_2\otimes b_3\otimes\ldots\otimes b_n\right)\right)=\\
\left(\left(a_2\otimes a_1\otimes a_3\otimes\ldots\otimes
a_n\right)\otimes \left(b_2\otimes b_1\otimes
b_3\otimes\ldots\otimes b_n\right)\right)\quad\forall\, a_1,\ldots
a_n,b_1,\ldots,b_n\in\mathbb{C}^{2}.
\end{multline}
Therefore, if $\xi\left(\vec x_1,\vec x_2,\vec x_3,\ldots,\vec
x_n,\vec y_1,\vec y_2,\vec y_3,\ldots,\vec
y_n,t\right)\in\mathbb{C}^{2^n}\otimes \mathbb{C}^{2^n}$ is a
solution of
\er{vhfffngghkjgghfjjghghghhjghjgghkghggkghghjghSYSPNjjjhhhjj} then
for every $t\geq 0$ we will have
\begin{multline}\label{fyfgjguyiuiHHgjgjgghbmgjgfhfhhHHrr}
B_{1,2}\cdot\xi(\vec x_2,\vec x_1,\vec x_3,\ldots,\vec x_n,\vec
y_2,\vec y_1,\vec y_3,\ldots,\vec y_n,t)=\xi\left(\vec x_1,\vec
x_2,\vec x_3,\ldots,\vec x_n,\vec y_1,\vec y_2,\vec y_3,\ldots,\vec
y_n,t\right)
\\ \quad\quad\forall\,\vec x_1,\vec x_2,\vec
x_3,\ldots,\vec x_n,\vec y_1,\vec y_2,\vec y_3,\ldots,\vec
y_n\in\mathbb{R}^3,
\end{multline}
provided that \er{fyfgjguyiuiHHgjgjgghbmgjgfhfhhHHrr} holds for the
initial instant of time $t=0$. So we have a consistency with the
Pauli Exclusion Principle for two or more identical fermions.

Next given $\xi\left(\vec x_1,\ldots,\vec x_n,\vec y_1,\ldots,\vec
y_n,t\right)$ that satisfies
\begin{equation}\label{gyfyfgfgfgghZZHHhuggkjhjjggjhhhjhgg22jjjkjkkljjjkkjopiiiojjkljkjkjj}
\xi\left(\vec x_1,\ldots,\vec x_n,\vec y_1,\ldots,\vec
y_n,t\right)=M\cdot\bar\xi\left(\vec y_1,\ldots,\vec y_n,\vec
x_1,\ldots,\vec x_n,t\right)\quad\forall\, t\geq 0,
\end{equation}
define the operator $\hat R_\xi$ by
\begin{multline}\label{gyfyfgfgfgghZZHHhuggkjhjjggjhhhjhgg22jjjkjkkljjjkkjopiiiojjkljkjkjkljkjkjj}
\hat R_\xi\left(\vec x_1,\ldots,\vec x_n,t\right)\cdot\psi\left(\vec
x_1,\ldots,\vec x_n,t\right)=\\ \int A\left(\xi\left(\vec
x_1,\ldots,\vec x_n,\vec y_1,\ldots,\vec y_n,t\right),\psi\left(\vec
y_1,\ldots,\vec y_n,t\right)\right)d\vec y_1\ldots d\vec y_n,
\end{multline}
where $A$ is a bilinear mapping, acting on the pairs of vectors in
$\left(\mathbb{C}^{2^n}\otimes \mathbb{C}^{2^n}\right)\times
\mathbb{C}^{2^n}$, taking the values in $\mathbb{C}^{2^n}$ and
having the form
\begin{equation}\label{fyfgjguyiuiHHHHHHkjjjjljhjhhhjhhjjhyhyujkjk}
A\left((a\otimes b), c\right)=(b\cdot c)a\quad\forall
a\in\mathbb{C}^{2^n},\;\forall b\in\mathbb{C}^{2^n},\;\forall
c\in\mathbb{C}^{2^n}.
\end{equation}
Then the mapping $\xi\to\hat R_\xi$ is one-to-one. Moreover, by
\er{gyfyfgfgfgghZZHHhuggkjhjjggjhhhjhgg22jjjkjkkljjjkkjopiiiojjkljkjkjj}
$\hat R_\xi$ is an Hermitian operator. Finally by
\er{vhfffngghkjgghfjjghghghhjghjgghkghggkghghjghSYSPNjjjhhhjj} we
have
\begin{multline}\label{vhfffngghkjgghfjjghghghhjghjgghkghggkghghjghSYSPNjjjhhhjjkjkjklkjljljj}
i\hbar\frac{\partial\hat R_\xi}{\partial t}\left(\vec
x_1,\ldots,\vec x_n,t\right)=\\ \hat H_0\left(\vec x_1,\ldots,\vec
x_n,t\right)\cdot\hat R_\xi\left(\vec x_1,\ldots,\vec
x_n,t\right)-\hat R_\xi\left(\vec x_1,\ldots,\vec
x_n,t\right)\cdot\hat H_0\left(\vec x_1,\ldots,\vec x_n,t\right).
\end{multline}
Equation
\er{vhfffngghkjgghfjjghghghhjghjgghkghggkghghjghSYSPNjjjhhhjjkjkjklkjljljj}
is equivalent to
\er{vhfffngghkjgghfjjghghghhjghjgghkghggkghghjghSYSPNjjjhhhjj}.

Next, clearly,  $\xi$ satisfies
\er{fyfgjguyiuiHHgjgjgghbmgjgfhfhhHHrr} if and only if we have the
following relation:
\begin{multline}\label{fyfgjguyiuiHHgjgjgghbmgjgfhfhhHHljjjhjhjhyyuyu}
A_{1,2}\cdot\psi(\vec x_2,\vec x_1,\vec x_3,\ldots,\vec
x_n,t)=-\psi(\vec x_1,\vec x_2,\vec x_3,\ldots,\vec
x_n,t)\;\;\forall\,\vec x_1,\ldots,\vec
x_n\in\mathbb{R}^3\quad\text{and}\quad \phi=\hat R_\xi\cdot\psi\\
\quad\text{implies}\quad A_{1,2}\cdot\phi(\vec x_2,\vec x_1,\vec
x_3,\ldots,\vec x_n,t)=-\phi(\vec x_1,\vec x_2,\vec x_3,\ldots,\vec
x_n,t)\;\;\forall\,\vec x_1,\ldots,\vec x_n\in\mathbb{R}^3,
\end{multline}
where by $A_{1,2}:\mathbb{C}^{2^n}\to \mathbb{C}^{2^n}$ we denote
the linear operator (matrix) defined as in
\er{fyfgjguyiuiHHgjgjgghhHH} by the following:
\begin{equation}\label{fyfgjguyiuiHHgjgjgghhHHhuhuhuh}
A_{1,2}\cdot\left(a_1\otimes a_2\otimes a_3\otimes\ldots\otimes
a_n\right)=\left(a_2\otimes a_1\otimes a_3\otimes\ldots\otimes
a_n\right)\quad\quad\forall\, a_1,\ldots a_n\in\mathbb{C}^{2}.
\end{equation}

Finally, assume that $\hat A\left(\vec x_1,\ldots,\vec x_n,t\right)$
is some Hermitian operator acting on the functions\\ $\psi\left(\vec
x_1,\ldots,\vec x_n,t\right)\in \mathbb{C}^{2^n}$ with respect to
spatial variables $\left(\vec x_1,\ldots,\vec x_n\right)$. Then, we
remind that the average $\tilde A(t)$ of $\hat A$ on the density
matrix $\xi$ is defined by \er{gyfyfgfgfgghZZHHhuggkjhjjgg}:
\begin{equation}\label{gyfyfgfgfgghZZHHhuggkjhjjggkhjhjh}
\tilde A(t)=\frac{\int trace\left(\left(\hat A\left(\vec
x_1,\ldots,\vec x_n,t\right)\otimes
I\right)\cdot\xi\right)\left(\vec z_1,\ldots,\vec z_n,\vec
z_1,\ldots,\vec z_n,t\right)d\vec z_1\ldots d\vec z_n}{\int
trace\left(\xi\right)\left(\vec z_1,\ldots,\vec z_n,\vec
z_1,\ldots,\vec z_n,t\right)d\vec z_1\ldots d\vec z_n}.
\end{equation}
Thus,
\begin{equation}\label{gyfyfgfgfgghZZHHhuggkjhjjgg22hjhjjhjhmnbmbhjhj}
trace \left(\hat A\circ \hat R_\xi\right)=\int trace\left(\left(\hat
A\left(\vec x_1,\ldots,\vec x_n,t\right)\otimes
I\right)\cdot\xi\right)\left(\vec z_1,\ldots,\vec z_n,\vec
z_1,\ldots,\vec z_n,t\right)d\vec z_1\ldots d\vec z_n,
\end{equation}
and so,
\begin{equation}\label{gyfyfgfgfgghZZHHhuggkjhjjgg22hjhjjhjhmnbmb}
\tilde A(t)=\frac{trace \left(\hat A\circ \hat R_\xi\right)}{trace
\left(\hat R_\xi\right)}=\frac{trace \left( \hat R_\xi\circ \hat
A\right)}{trace \left(\hat R_\xi\right)},
\end{equation}
where by the trace we mean the trace of an operator on a Hilbert
space. Moreover, by \er{gyfyfgfgfgghZZHHhuggkjhjjggjhhhjhgg} and
\er{gyfyfgfgfgghZZHHhuggkjhjjgg22hjhjjhjhmnbmbhjhj} we have
\begin{equation}\label{gyfyfgfgfgghZZHHhuggkjhjjgg22hjhjjhjhmnbmbjkjhjhjhh}
trace \left( \hat R_\xi \right)(t)=trace \left( \hat
R_\xi\right)(0).
\end{equation}

\subsubsection{Thermodynamical equilibrium in systems of spin-half particles; canonical
ensemble} Let $\xi\left(\vec x_1,\ldots,\vec x_n,\vec
y_1,\ldots,\vec y_n,t\right)$ be the density matrix of the quantum
system, ruled by the Hamiltonian operator $\hat H_0$ from
\er{vhfffngghkjgghfjjghghghhjghjgghkghgghjhggjjkgSYSPNjj11}, in the
case of \underline{thermodynamical equilibrium} and $\hat R_\xi$ is
given by
\er{gyfyfgfgfgghZZHHhuggkjhjjggjhhhjhgg22jjjkjkkljjjkkjopiiiojjkljkjkjkljkjkjj}
for this $\xi$. Then in the case that $\frac{\partial \hat
H_0}{\partial t}\equiv 0$ and the given equilibrium system rests
macroscopically, i.e. it has the macroscopical velocity field zero:
$\vec u(\vec x,t)\equiv 0$ it is well known that
\begin{equation}\label{vhfffngghkjgghfjjghghghSYSPNjljllkljjhkhkhklpl;lpoopkklkljlrrccnn}
\hat R_\xi=\frac{1}{trace \left( f\circ \hat H_0 \right)} f\circ
\hat H_0,
\end{equation}
and in the case of a system in thermostat, having the Kelvin
temperature $T$, we have the canonical ensemble:
\begin{equation}\label{vhfffngghkjgghfjjghghghSYSPNjljllkljjhkhkhklpl;lpoopkklkljlrrcchhkhccnn}
f(s)=e^{-\frac{s}{kT}}=\sum_{m=0}^{+\infty}\frac{1}{m!}\left(-\frac{s}{kT}\right)^m\quad\forall
s\in\mathbb{C},
\end{equation}
where $k$ is the Boltzmann constant and by $f\circ \hat H_0$ we
denote the following operator
\begin{equation}\label{vhfffngghkjgghfjjghghghSYSPNjljllkljjhkhkhklpl;lpoopkklkljlrrcchhkhcchjhjgccnn}
f\circ \hat
H_0=\sum_{m=0}^{+\infty}\frac{1}{m!}\left(-\frac{1}{kT}\hat
H_0\right)^m\quad\forall s\in\mathbb{C}.
\end{equation}
We would like to find alternative form of the above law of
thermodynamical equilibrium, which is invariant under the change of
inertial or non-inertial cartesian coordinate system. Then it is
clear that if the given system rests in the old coordinate system,
then it obviously has non-trivial macroscopical velocity field $\vec
u(\vec x,t)$ in the new one. Moreover, $\vec u(\vec x,t)$ can depend
on $\vec x$ and $t$, as it indeed happens in the case of a rotation
of the new coordinate system with respect to the old one. On the
other hand the concept of thermodynamical equilibrium is clearly
independent of the coordinate system.

In order to find the forms of the law of thermodynamical
equilibrium, which are indeed invariant under the change of inertial
or non-inertial cartesian coordinate system we follow the steps as
below. Given a speed-like vector field $\vec u=\vec u(\vec x,t)$,
define
\begin{multline}\label{vhfffngghkjgghfjjghghghSYSPNjljll'ljjjjjjhkjhjhjjkjkkjj}
\hat H_{\vec u}\left(\vec x_1,\ldots,\vec x_n,t\right):=\\
\hat H_0\left(\vec x_1,\ldots,\vec
x_n,t\right)-\frac{1}{2}\left(\sum_{j=1}^{n}\hat {\vec P}_j\cdot\vec
u(\vec x_j,t)+\vec u(\vec x_j,t)\cdot \hat {\vec
P}_j\right)-\sum_{j=1}^{n}\frac{\hbar}{4}\vec
S_j\cdot\left(curl_{\vec x_j}\vec u(\vec x_j,t)\right),
\end{multline}
where $\hat{\vec P}_j$ is operator defined as
\begin{equation}\label{vhfffngghkjgghfjjghghghSYSPNjljll'ljjjjjjhkjhjhjjkjkkljjljjkkjj}
\hat {\vec P}_j\cdot\psi\left(\vec x_1,\ldots,\vec
x_n\right)=-i\hbar\nabla_{\vec x_j}\psi\left(\vec x_1,\ldots,\vec
x_n\right).
\end{equation}
Then by \er{vhfffngghkjgghfjjghghghhjghjgghkghgghjhggjjkgSYSPNjj11}
we have
\begin{multline}\label{vhfffngghkjgghfjjghghghhjghjgghkghgghjhggjjkgSYSPNjjjljjkljjkjkjj}
\hat H_{\vec u}\left(\vec x_1,\ldots,\vec
x_n,t\right)\cdot\psi\left(\vec x_1,\ldots,\vec x_n,t\right)=
-\sum_{j=1}^{n}\frac{\hbar^2}{2m_j}\Delta_{\vec
x_j}\psi-\sum_{j=1}^{n}\frac{i\hbar}{2}div_{\vec
x_j}\left\{\psi\left(\vec v(\vec x_j,t)-\vec u(\vec
x_j,t)\right)\right\}\\-\sum_{j=1}^{n}\frac{i\hbar}{2}\left(\vec
v(\vec x_j,t)-\vec u(\vec x_j,t)\right)\cdot\nabla_{\vec
x_j}\psi+\sum_{j=1}^{n}\frac{ i\hbar\sigma_j}{2m_jc}div_{\vec
x_j}\left\{\psi\vec A(\vec x_j,t)\right\}+\sum_{j=1}^{n}\frac{
i\hbar\sigma_j}{2m_jc}\vec A(\vec x_j,t)\cdot\nabla_{\vec
x_j}\psi\\+\sum_{j=1}^{n}\left(\sigma_j\Psi(\vec
x_j,t)-\frac{\sigma_j}{c}\vec A(\vec x_j,t)\cdot\vec v(\vec
x_j,t)+\frac{\sigma^2_j}{2m_jc^2}\left|\vec A(\vec
x_j,t)\right|^2\right)\psi-V\left(\vec x_1,\ldots,\vec
x_n,t\right)\psi\\-\sum_{j=1}^{n}\frac{g_j\sigma_j\hbar}{2m_jc}\vec
S_j\cdot\left(curl_{\vec x_j}\vec A(\vec
x_j,t)\,\psi\right)+\sum_{j=1}^{n}\frac{\hbar}{4}\vec
S_j\cdot\left(curl_{\vec x_j}\left(\vec v(\vec x_j,t)-\vec u(\vec
x_j,t)\right)\,\psi\right).
\end{multline}
Thus, as before, we can prove that
\er{vhfffngghkjgghfjjghghghhjghjgghkghgghjhggjjkgSYSPNjjjljjkljjkjkjj}
is invariant under the change of inertial or non-inertial cartesian
coordinate systems i.e.
\begin{multline}\label{vhfffngghkjgghfjjghghghhjghjgghkghgghjhggjjkgSYSPNjjjljjkljjkjkjjz11}
\hat H'_{\vec u'}\left(\vec x'_1,\ldots,\vec
x'_n,t'\right)\cdot\psi'\left(\vec x'_1,\ldots,\vec
x'_n,t'\right)=\phi'\left(\vec x'_1,\ldots,\vec x'_n,t'\right)
\quad\text{if and only if}\quad\\  \hat H_{\vec u}\left(\vec
x_1,\ldots,\vec x_n,t\right)\cdot\psi\left(\vec x_1,\ldots,\vec
x_n,t\right)=\phi\left(\vec x_1,\ldots,\vec x_n,t\right),
\end{multline}
provided that, as before in
\er{yuythfgfyftydtydtydtyddyyyhhddhhhredPPN111hgghjgintintrrZZHH},
we have
\begin{equation}\label{yuythfgfyftydtydtydtyddyyyhhddhhhredPPN111hgghjgintintrrZZHHhkkhjhj11}
\begin{cases}
V'\left(\vec x'_1,\ldots,\vec x'_n,t'\right)\,=\,V\left(\vec
x_1,\ldots,\vec x_n,t\right),\\
\sigma'_j\,=\,\sigma_j,\\
m'_j\,=\,m_j,\\
g'_j\,=\,g_j,\\
\vec v'(\vec x',t)\,=\,A(t)\cdot \vec v(\vec x,t)+\frac{dA}{dt}(t)\cdot\vec x+\frac{d\vec z}{dt}(t),\\
\vec u'(\vec x',t)\,=\,A(t)\cdot \vec u(\vec x,t)+\frac{dA}{dt}(t)\cdot\vec x+\frac{d\vec z}{dt}(t),\\
\vec A'(\vec x',t)\,=\,A(t)\cdot \vec A(\vec x,t),\\
\Psi'(\vec x',t)-\vec v'(\vec x',t)\cdot\vec A'(\vec x',t)\,=\,\Psi(\vec x,t)-\vec v(\vec x,t)\cdot\vec A(\vec x,t),\\
\psi'\left(\vec x'_1,\ldots,\vec
x'_n,t'\right)\,=\,\left(U(t)\otimes_{1}U(t)\otimes_{2}U(t)\ldots\otimes_{(n-1)}U(t)
\right)\cdot\psi\left(\vec x_1,\ldots,\vec x_n,t\right),\\
\phi'\left(\vec x'_1,\ldots,\vec
x'_n,t'\right)\,=\,\left(U(t)\otimes_{1}U(t)\otimes_{2}U(t)\ldots\otimes_{(n-1)}U(t)
\right)\cdot\phi\left(\vec x_1,\ldots,\vec x_n,t\right).
\end{cases}
\end{equation}
%
%
%
%
%
%
Next consider
\begin{equation}\label{vhfffngghkjgghfjjghghghSYSPNjljllkljjhkhkhklpl;lpooprrccnn}
\hat R_\xi=\frac{1}{trace \left( f\circ \hat H_{\vec u} \right)}
f\circ \hat H_{\vec u},
\end{equation}
where, as before, in the case of a system in thermostat, having the
Kelvin temperature $T$, we have the canonical ensemble:
\begin{equation}\label{vhfffngghkjgghfjjghghghSYSPNjljllkljjhkhkhklpl;lpoopkklkljlrrcchhkhcchhhjjkljccnn}
f(s)=e^{-\frac{s}{kT}}=\sum_{m=0}^{+\infty}\frac{1}{m!}\left(-\frac{s}{kT}\right)^m\quad\forall
s\in\mathbb{C}
\end{equation}
with $k$ being the Boltzmann constant and $f\circ \hat H_{\vec u}$
is the following operator
\begin{equation}\label{vhfffngghkjgghfjjghghghSYSPNjljllkljjhkhkhklpl;lpoopkklkljlrrcchhkhcchjhjgcckjhhjccnn}
f\circ \hat H_{\vec
u}=\sum_{m=0}^{+\infty}\frac{1}{m!}\left(-\frac{1}{kT}\hat H_{\vec
u}\right)^m\quad\forall s\in\mathbb{C}.
\end{equation}
We would like to note here that since the operator $ \hat H_{\vec
u}$ is bounded from below and $f(s)$ decays rapidly as
$s\to+\infty$, there exists a density matrix $\xi_1$, such that the
operator $\hat R_{\xi_1}$, given by
\er{gyfyfgfgfgghZZHHhuggkjhjjggjhhhjhgg22jjjkjkkljjjkkjopiiiojjkljkjkjkljkjkjj},
equals to the operator $f\circ \hat H_{\vec u}$ and thus
\er{vhfffngghkjgghfjjghghghSYSPNjljllkljjhkhkhklpl;lpooprrccnn}
indeed has sense. Next by
\er{vhfffngghkjgghfjjghghghhjghjgghkghgghjhggjjkgSYSPNjjjljjkljjkjkjjz11}
and
\er{vhfffngghkjgghfjjghghghSYSPNjljllkljjhkhkhklpl;lpoopkklkljlrrcchhkhcchjhjgcckjhhjccnn}
the operator $\hat R_\xi$ in the left hand side of
\er{vhfffngghkjgghfjjghghghSYSPNjljllkljjhkhkhklpl;lpooprrccnn} is
invariant under the change of inertial or non-inertial cartesian
coordinate systems i.e.
\begin{multline}\label{vhfffngghkjgghfjjghghghhjghjgghkghgghjhggjjkgSYSPNjjjljjkljjkjkjjz1155kjjkjknn}
\hat R_{\xi'}\left(\vec x'_1,\ldots,\vec
x'_n,t'\right)\cdot\psi'\left(\vec x'_1,\ldots,\vec
x'_n,t'\right)=\phi'\left(\vec x'_1,\ldots,\vec x'_n,t'\right)
\quad\text{if and only if}\quad\\  \hat R_\xi\left(\vec
x_1,\ldots,\vec x_n,t\right)\cdot\psi\left(\vec x_1,\ldots,\vec
x_n,t\right)=\phi\left(\vec x_1,\ldots,\vec x_n,t\right),
\end{multline}
provided that, as before, we have
\er{yuythfgfyftydtydtydtyddyyyhhddhhhredPPN111hgghjgintintrrZZHHhkkhjhj11}
and
\begin{multline}\label{vyfgjhgjhvhgghSYSPNjj11jhjggjjhjg}
\xi'\left(\vec x'_1,\ldots,\vec x'_n,\vec y'_1,\ldots,\vec
y'_n,t'\right)=\\
\left(\left(U(t)\otimes_{1}U(t)\ldots\otimes_{(n-1)}U(t)
\right)\otimes\left(\bar U(t)\otimes_{1}\bar
U(t)\ldots\otimes_{(n-1)}\bar U(t) \right)\right)\cdot\xi\left(\vec
x_1,\ldots,\vec x_n,\vec y_1,\ldots,\vec y_n,t\right).
\end{multline}
Moreover, in the case $\vec u\equiv 0$
\er{vhfffngghkjgghfjjghghghSYSPNjljllkljjhkhkhklpl;lpooprrccnn}
coincides with
\er{vhfffngghkjgghfjjghghghSYSPNjljllkljjhkhkhklpl;lpoopkklkljlrrccnn}.
However, we still need to derive the restrictions on the field $\vec
u$ and the Hamiltonian operator $\hat H_0$, providing that our
system can indeed be found in the state of thermodynamical
equilibrium. We remind that in the case $\vec u\equiv 0$ the
appropriate restriction is $\frac{\partial \hat H_0}{\partial
t}\equiv 0$. In order to get these restrictions in the general case,
we need to insert $\hat R_\xi$ in
\er{vhfffngghkjgghfjjghghghSYSPNjljllkljjhkhkhklpl;lpooprrccnn} into
the equation in
\er{vhfffngghkjgghfjjghghghhjghjgghkghggkghghjghSYSPNjjjhhhjjkjkjklkjljljj}
which is equivalent to the quantum Liouville equation.

Therefore, assume that the vector field $\vec u$ satisfies
\er{vhfffngghkjgghfjjghghghhjghjgghkghggSYSPNkljljjkklmjkjkjjk,mmm,kl;k;lklkjkjkkjkjkjkjljljjljklggghgkhhkhkhkkklkklklkljkhkhjjkhjh}
with \er{vhfffngghkjgghfjjghghghSYSPNjljll'lhkhhjgggjgjgghgjkjjk}
i.e.
\begin{equation}\label{vhfffngghkjgghfjjghghghhjghjgghkghggSYSPNkljljjkklmjkjkjjk,mmm,kl;k;lklkjkjkkjkjkjkjljljjljklggghgkhhkhkhkkklkklklkljkhkhjjkhjhmnnmmn11}
\begin{cases}
\frac{\partial U}{\partial t}\left(\vec x_1,\ldots,\vec
x_n,t\right)+\sum_{j=1}^{n}\vec u( \vec x_j,t)\cdot\nabla_{\vec
x_j}U\left(\vec x_1,\ldots,\vec x_n,t\right)=0,\\
\frac{\partial}{\partial t}\left(\vec u(\vec x,t)-\vec v(\vec
x,t)\right)-curl_{\vec x}\left(\vec u(\vec x,t)\times\left(\vec
u(\vec x,t)-\vec v(\vec x,t)\right)\right)+\left(\Div_{\vec
x}\left(\vec u(\vec x,t)-\vec v(\vec x,t)\right)\right)\vec
u(\vec x,t)=0,\\
\frac{\partial\vec A}{\partial t}(\vec x,t)-curl_{\vec x}\left(\vec
u(\vec x,t)\times\vec A(\vec x,t)\right)+\left(\Div_{\vec x}\vec
A(\vec x,t)\right)\vec u(\vec x,t)=0,
\\
d_{\vec x}\vec u(\vec x,t)+\left\{d_{\vec x}\vec u(\vec
x,t)\right\}^T=0,
\end{cases}
\end{equation}
where
\begin{multline}\label{vhfffngghkjgghfjjghghghSYSPNjljll'lhkhhjgggjgjgghgjkjjkjkjk11}
U\left(\vec x_1,\ldots,\vec x_n,t\right):=\\
\sum_{j=1}^{n}\left(\frac{\sigma^2_j}{2m_jc^2}\left|\vec A(\vec
x_j,t)\right|^2+\sigma_j\Psi(\vec x_j,t)-\frac{\sigma_j}{c}\vec
A(\vec x_j,t)\cdot\vec v(\vec x_j,t)\right)-V\left(\vec
x_1,\ldots,\vec x_n,t\right).
\end{multline}
Then, by Proposition \ref{eqv2} below we have
\begin{equation}\label{vhfffngghkjgghfjjghghghhjghjgghkghggkghghjghSYSPNjjjhhhjjkjkjklkjljljkkhhkjiijjijhbmbgghghcchkhhnn}
i\hbar\frac{\partial}{\partial t}\left(f\circ \hat H_{\vec
u}\right)= \hat H_0\cdot\left(f\circ \hat H_{\vec
u}\right)-\left(f\circ \hat H_{\vec u}\right)\cdot\hat H_0,
\end{equation}
Then by
\er{vhfffngghkjgghfjjghghghhjghjgghkghggkghghjghSYSPNjjjhhhjjkjkjklkjljljkkhhkjiijjijhbmbgghghcchkhhnn}
together with
\er{gyfyfgfgfgghZZHHhuggkjhjjgg22hjhjjhjhmnbmbjkjhjhjhh} we deduce
that
\begin{equation}\label{vhfffngghkjgghfjjghghghhjghjgghkghggkghghjghSYSPNjjjhhhjjkjkjklkjljljkkhhkjiijjijhbmbgghghcchkhhjkjkknn}
i\hbar\frac{\partial}{\partial t}\left(\hat R_\xi\right)= \hat
H_0\cdot\left(\hat R_\xi\right)-\left(\hat R_\xi\right)\cdot\hat
H_0,
\end{equation}
where $\hat R_\xi$ is the operator in the left hand side of
\er{vhfffngghkjgghfjjghghghSYSPNjljllkljjhkhkhklpl;lpooprrccnn}. So
we indeed get
\er{vhfffngghkjgghfjjghghghhjghjgghkghggkghghjghSYSPNjjjhhhjjkjkjklkjljljj}
in the case of
\er{vhfffngghkjgghfjjghghghhjghjgghkghggSYSPNkljljjkklmjkjkjjk,mmm,kl;k;lklkjkjkkjkjkjkjljljjljklggghgkhhkhkhkkklkklklkljkhkhjjkhjhmnnmmn11},
\er{vhfffngghkjgghfjjghghghSYSPNjljll'lhkhhjgggjgjgghgjkjjkjkjk11}.
Moreover,
\er{vhfffngghkjgghfjjghghghhjghjgghkghggSYSPNkljljjkklmjkjkjjk,mmm,kl;k;lklkjkjkkjkjkjkjljljjljklggghgkhhkhkhkkklkklklkljkhkhjjkhjhmnnmmn11},
\er{vhfffngghkjgghfjjghghghSYSPNjljll'lhkhhjgggjgjgghgjkjjkjkjk11}
are invariant under the change of inertial or non-inertial
coordinate system.

\begin{proposition}\label{eqv2}
Assume that the speed-like vector field  $\vec u$ satisfies
\er{vhfffngghkjgghfjjghghghhjghjgghkghggSYSPNkljljjkklmjkjkjjk,mmm,kl;k;lklkjkjkkjkjkjkjljljjljklggghgkhhkhkhkkklkklklkljkhkhjjkhjhmnnmmn11}
and
\er{vhfffngghkjgghfjjghghghSYSPNjljll'lhkhhjgggjgjgghgjkjjkjkjk11}.
Next assume that the holomorphic function $f(z):\mathbb{C}\to
\mathbb{C}$ defined as a sum of the power series
\begin{equation}\label{vhfffngghkjgghfjjghghghhjghjgghkghggkghghjghSYSPNjjjhhhjjkjkjklkjljljkkhhkjiijjihhjjhjhjhjhkkhhk11cc}
f(z):=\sum_{m=0}^{+\infty}a_mz^m,
\end{equation}
with $a_m\in\mathbb{C}$, is such that for the operator $f\circ \hat
H_{\vec u}$, given by:
\begin{equation}\label{vhfffngghkjgghfjjghghghhjghjgghkghggkghghjghSYSPNjjjhhhjjkjkjklkjljljkkhhkjiijjihhjjhjhjhjhkkhhknnhkhh11cc}
f\circ \hat H_{\vec u}:=\sum_{m=0}^{+\infty}a_m\hat H^m_{\vec u},
\end{equation}
where the operator $\hat H_{\vec u}$  is given by
\er{vhfffngghkjgghfjjghghghhjghjgghkghgghjhggjjkgSYSPNjjjljjkljjkjkjj},
there exists a density matrix $\xi\left(\vec x_1,\ldots,\vec
x_n,\vec y_1,\ldots,\vec y_n,t\right)$, such that
\begin{equation}\label{vhfffngghkjgghfjjghghghhjghjgghkghggkghghjghSYSPNjjjhhhjjkjkjklkjljljkkhhkjiijjijhbmbgghgh11ggghhgfccjggjgjhcc}
\left(f\circ\hat H_{\vec u}\right)\left(\vec x_1,\ldots,\vec
x_n,t\right)=\hat R_{\xi}\left(\vec x_1,\ldots,\vec x_n,t\right),
\end{equation}
where $\hat R_\xi$ is given by
\er{gyfyfgfgfgghZZHHhuggkjhjjggjhhhjhgg22jjjkjkkljjjkkjopiiiojjkljkjkjkljkjkjj}.
Then the operator $f\circ \hat H_{\vec u}$ satisfies:
\begin{equation}\label{vhfffngghkjgghfjjghghghhjghjgghkghggkghghjghSYSPNjjjhhhjjkjkjklkjljljkkhhkjiijjijhbmbgghghihihhcc}
i\hbar\frac{\partial}{\partial t}\left(f\circ \hat H_{\vec
u}\right)= \hat H_0\cdot\left(f\circ \hat H_{\vec
u}\right)-\left(f\circ \hat H_{\vec u}\right)\cdot\hat H_0,
\end{equation}
or equivalently
\begin{multline}\label{vhfffngghkjgghfjjghghghhjghjgghkghggkghghjghSYSPNjjjhhhjjcc}
i\hbar\frac{\partial\xi}{\partial t}\left(\vec x_1,\ldots,\vec
x_n,\vec y_1,\ldots,\vec y_n,t\right)=\\ \left(\hat H_0\left(\vec
x_1,\ldots,\vec x_n,t\right)\otimes I\right)\cdot\xi\left(\vec
x_1,\ldots,\vec x_n,\vec y_1,\ldots,\vec
y_n,t\right)\\-\left(I\otimes\hat H^*_0\left(\vec y_1,\ldots,\vec
y_n,t\right)\right)\cdot\xi\left(\vec x_1,\ldots,\vec x_n,\vec
y_1,\ldots,\vec y_n,t\right).
\end{multline}
Moreover, $f\circ \hat H_{\vec u}$ is invariant under the change of
inertial or non-inertial cartesian coordinate systems i.e.
\begin{multline}\label{vhfffngghkjgghfjjghghghhjghjgghkghgghjhggjjkgSYSPNjjjljjkljjkjkjjzjjjkkhhklhhlcciyiyyyycc}
\left(f\circ\hat H'_{\vec u'}\right)\left(\vec x'_1,\ldots,\vec
x'_n,t'\right)\cdot\psi'\left(\vec x'_1,\ldots,\vec
x'_n,t'\right)=\phi'\left(\vec x'_1,\ldots,\vec x'_n,t'\right)
\quad\text{if and only if}\quad\\  \left(f\circ\hat H_{\vec
u}\right)\left(\vec x_1,\ldots,\vec x_n,t\right)\cdot\psi\left(\vec
x_1,\ldots,\vec x_n,t\right)=\phi\left(\vec x_1,\ldots,\vec
x_n,t\right),
\end{multline}
provided that we have
\er{yuythfgfyftydtydtydtyddyyyhhddhhhredPPN111hgghjgintintrrZZHHhkkhjhj11}.
Finally, if we assume that the first and the second particles have
the same mass $m_1=m_2$ and the same charge $\sigma_1=\sigma_2$ and
moreover, if we assume that
$$V\left(\vec x_1,\vec x_2,\vec x_3,\ldots,\vec
x_n,t\right)=V\left(\vec x_2,\vec x_1,\vec x_3,\ldots,\vec
x_n,t\right)\quad\quad\forall\,\vec x_1,\vec x_2,\vec
x_3,\ldots,\vec x_n\in\mathbb{R}^3,\quad\forall t,$$ then
\begin{multline}\label{fyfgjguyiuiHHgjgjgghbmgjgfhfhhHHljjjhjhjhyyuyuhhhljjjkjjccbb}
A_{1,2}\cdot\psi(\vec x_2,\vec x_1,\vec x_3,\ldots,\vec
x_n,t)=-\psi(\vec x_1,\vec x_2,\vec x_3,\ldots,\vec
x_n,t)\;\;\forall\,\vec x_1,\ldots,\vec
x_n\in\mathbb{R}^3\quad\text{and}\quad \phi=\left(f\circ\hat H_{\vec u}\right)\cdot\psi\\
\quad\text{implies}\quad A_{1,2}\cdot\phi(\vec x_2,\vec x_1,\vec
x_3,\ldots,\vec x_n,t)=-\phi(\vec x_1,\vec x_2,\vec x_3,\ldots,\vec
x_n,t)\;\;\forall\,\vec x_1,\ldots,\vec x_n\in\mathbb{R}^3,
\end{multline}
where, as before, by $A_{1,2}:\mathbb{C}^{2^n}\to \mathbb{C}^{2^n}$
we denote the linear operator (matrix) defined as in
\er{fyfgjguyiuiHHgjgjgghhHH} by the following:
\begin{equation}\label{fyfgjguyiuiHHgjgjgghhHHhuhuhuhhkhcc}
A_{1,2}\cdot\left(a_1\otimes a_2\otimes a_3\otimes\ldots\otimes
a_n\right)=\left(a_2\otimes a_1\otimes a_3\otimes\ldots\otimes
a_n\right)\quad\quad\forall\, a_1,\ldots a_n\in\mathbb{C}^{2}.
\end{equation}
\end{proposition}
\begin{proof}

Again by
\er{vhfffngghkjgghfjjghghghhjghjgghkghggSYSPNkljljjkklmjkjkjjk,mmm,kl;k;lklkjkjkkjkjkjkjljljjljklggghgkhhkhkhkkklkklklkljkhkhjjkhjhmnnmmn11}
and Proposition \ref{jhjkhhjhjhgjgjjkjk} there exists another
cartesian coordinate system $(**)$ such that under the change of
coordinate system $(*)$ to another cartesian coordinate system
$(**)$, given by \er{noninchgravortbstrjgghguittu2hhhhh}, we have
\begin{equation}
\label{jljjjjkhhhjkkww11}A(t)\cdot \vec u(\vec x,t)+
\frac{d A}{dt}(t)\cdot\vec x+
\frac{d\vec z}{dt}(t)=\vec u'(\vec x',t')=0.
\end{equation}
Thus, since
\er{vhfffngghkjgghfjjghghghhjghjgghkghggSYSPNkljljjkklmjkjkjjk,mmm,kl;k;lklkjkjkkjkjkjkjljljjljklggghgkhhkhkhkkklkklklkljkhkhjjkhjhmnnmmn11},
\er{vhfffngghkjgghfjjghghghSYSPNjljll'lhkhhjgggjgjgghgjkjjkjkjk11}
are invariant under the change of inertial or non-inertial cartesian
coordinate systems, as before, in system $(**)$ we have
\begin{equation}\label{vhfffngghkjgghfjjghghghhjghjgghkghggSYSPNkljljjkklmjkjkjjk,mmm,kl;k;lklkjkjkkjkjkjkjljljjljklggghgkhhkhkhkkklkklklkljkhkhjjkhjhmjmm,mjhjkjkhjhj11}
\begin{cases}
\frac{\partial U'}{\partial t'}\left(\vec x'_1,\ldots,\vec
x'_n,t'\right)=0,\\
\frac{\partial\vec v'}{\partial t'}(\vec
x',t')=0,\\
\frac{\partial\vec A'}{\partial t'}(\vec x',t')=0,
\\
\vec u'(\vec x',t')=0,
\end{cases}
\end{equation}
where
\begin{multline}\label{vhfffngghkjgghfjjghghghSYSPNjljll'lhkhhjgggjgjgghgjkjjkkjjhj11}
U'\left(\vec x'_1,\ldots,\vec x'_n,t'\right):=\\
\sum_{j=1}^{n}\left(\frac{(\sigma'_j)^2}{2m'_jc^2}\left|\vec A'(\vec
x'_j,t')\right|^2+\sigma'_j\Psi'(\vec
x'_j,t')-\frac{\sigma'_j}{c}\vec A'(\vec x'_j,t')\cdot\vec v'(\vec
x'_j,t')\right)-V'\left(\vec x'_1,\ldots,\vec x'_n,t'\right).
\end{multline}
On the other hand,
\er{vhfffngghkjgghfjjghghghhjghjgghkghggSYSPNkljljjkklmjkjkjjk,mmm,kl;k;lklkjkjkkjkjkjkjljljjljklggghgkhhkhkhkkklkklklkljkhkhjjkhjhmjmm,mjhjkjkhjhj11}
is equivalent to
\begin{equation}\label{vhfffnhhjhjjkkjhkjkhkh11}
\frac{\partial \hat H'_0}{\partial t'}\left(\vec x'_1,\ldots,\vec
x'_n,t'\right)=0\quad\text{and}\quad\vec u'(\vec x',t')=0,
\end{equation}
in system $(**)$. Thus, by \er{vhfffnhhjhjjkkjhkjkhkh11}, in system
$(**)$ we deduce:
\begin{equation}\label{vhfffngghkjgghfjjghghghhjghjgghkghggkghghjghSYSPNjjjhhhjjkjkjklkjljljkkhhkjjjkljjbbb11}
\hat H'_0\cdot\hat H'_{\vec u'}-\hat H'_{\vec u'}\cdot\hat H'_0=\hat
H'_0\cdot\hat H'_0-\hat H'_0\cdot\hat
H'_0=0=i\hbar\frac{\partial\hat H'_0}{\partial
t'}=i\hbar\frac{\partial\hat H'_{\vec u'}}{\partial t'}.
\end{equation}
So we get:
\begin{equation}\label{vhfffngghkjgghfjjghghghhjghjgghkghggkghghjghSYSPNjjjhhhjjkjkjklkjljljkkhhkjjjkljjbbbjkljljl11}
i\hbar\frac{\partial\hat H'_{\vec u'}}{\partial t'}=\hat
H'_0\cdot\hat H'_{\vec u'}-\hat H'_{\vec u'}\cdot\hat H'_0.
\end{equation}
Therefore, given the holomorphic function $f(z):\mathbb{C}\to
\mathbb{C}$ defined as a sum of the power series
\begin{equation}\label{vhfffngghkjgghfjjghghghhjghjgghkghggkghghjghSYSPNjjjhhhjjkjkjklkjljljkkhhkjiijjihhjjhjhjhjhkkhhk11}
f(z):=\sum_{m=0}^{+\infty}a_mz^m,
\end{equation}
with $a_m\in\mathbb{C}$, if we define the operator $f\circ \hat
H_{\vec u}$ as:
\begin{equation}\label{vhfffngghkjgghfjjghghghhjghjgghkghggkghghjghSYSPNjjjhhhjjkjkjklkjljljkkhhkjiijjihhjjhjhjhjhkkhhknnhkhh11}
f\circ \hat H_{\vec u}:=\sum_{m=0}^{+\infty}a_m\hat H^m_{\vec u},
\end{equation}
by
\er{vhfffngghkjgghfjjghghghhjghjgghkghggkghghjghSYSPNjjjhhhjjkjkjklkjljljkkhhkjjjkljjbbbjkljljl11}
and Proposition \ref{hugyfyfffgh} we deduce
\begin{equation}\label{vhfffngghkjgghfjjghghghhjghjgghkghggkghghjghSYSPNjjjhhhjjkjkjklkjljljkkhhkjiijjijhbmbgghghdd11}
i\hbar\frac{\partial}{\partial t'}\left(f\circ \hat H'_{\vec
u'}\right)= \hat H'_0\cdot\left(f\circ \hat H'_{\vec
u'}\right)-\left(f\circ \hat H'_{\vec u'}\right)\cdot\hat H'_0.
\end{equation}
Moreover, by
\er{vhfffngghkjgghfjjghghghhjghjgghkghgghjhggjjkgSYSPNjjjljjkljjkjkjjz11}
and
\er{vhfffngghkjgghfjjghghghhjghjgghkghggkghghjghSYSPNjjjhhhjjkjkjklkjljljkkhhkjiijjihhjjhjhjhjhkkhhknnhkhh11},
we can easily prove that $f\circ \hat H_{\vec u}$ is invariant under
the change of inertial or non-inertial cartesian coordinate systems
i.e.
\begin{multline}\label{vhfffngghkjgghfjjghghghhjghjgghkghgghjhggjjkgSYSPNjjjljjkljjkjkjjzjjjkkhhklhhl11}
\left(f\circ\hat H'_{\vec u'}\right)\left(\vec x'_1,\ldots,\vec
x'_n,t'\right)\cdot\psi'\left(\vec x'_1,\ldots,\vec
x'_n,t'\right)=\phi'\left(\vec x'_1,\ldots,\vec x'_n,t'\right)
\quad\text{if and only if}\quad\\  \left(f\circ\hat H_{\vec
u}\right)\left(\vec x_1,\ldots,\vec x_n,t\right)\cdot\psi\left(\vec
x_1,\ldots,\vec x_n,t\right)=\phi\left(\vec x_1,\ldots,\vec
x_n,t\right),
\end{multline}
provided that, as before, we have
\er{yuythfgfyftydtydtydtyddyyyhhddhhhredPPN111hgghjgintintrrZZHHhkkhjhj11}.
Next assume that the holomorphic function $f$ in
\er{vhfffngghkjgghfjjghghghhjghjgghkghggkghghjghSYSPNjjjhhhjjkjkjklkjljljkkhhkjiijjihhjjhjhjhjhkkhhk11}
is such that there exists a density matrix $\xi\left(\vec
x_1,\ldots,\vec x_n,\vec y_1,\ldots,\vec y_n,t\right)$ satisfying
\begin{equation}\label{vhfffngghkjgghfjjghghghhjghjgghkghggkghghjghSYSPNjjjhhhjjkjkjklkjljljkkhhkjiijjijhbmbgghgh11ggghhgf11gg11}
\left(f\circ\hat H_{\vec u}\right)\left(\vec x_1,\ldots,\vec
x_n,t\right)=\hat R_{\xi}\left(\vec x_1,\ldots,\vec x_n,t\right),
\end{equation}
where $\hat R_\xi$ is given by
\er{gyfyfgfgfgghZZHHhuggkjhjjggjhhhjhgg22jjjkjkkljjjkkjopiiiojjkljkjkjkljkjkjj}.
Then, by
\er{vhfffngghkjgghfjjghghghhjghjgghkghgghjhggjjkgSYSPNjjjljjkljjkjkjjzjjjkkhhklhhl11}
for the density matrix $\xi'\left(\vec x'_1,\ldots,\vec x'_n,\vec
y'_1,\ldots,\vec y'_n,t'\right)$ we have
\begin{equation}\label{vhfffngghkjgghfjjghghghhjghjgghkghggkghghjghSYSPNjjjhhhjjkjkjklkjljljkkhhkjiijjijhbmbgghgh11ggghhgf11}
\left(f\circ\hat H'_{\vec u'}\right)\left(\vec x'_1,\ldots,\vec
x'_n,t'\right)=\hat R_{\xi'}\left(\vec x'_1,\ldots,\vec
x'_n,t'\right),
\end{equation}
provided that
\begin{multline}\label{vyfgjhgjhvhgghSYSPNjj11jhjggj}
\xi'\left(\vec x'_1,\ldots,\vec x'_n,\vec y'_1,\ldots,\vec
y'_n,t'\right)=\\
\left(\left(U(t)\otimes_{1}U(t)\ldots\otimes_{(n-1)}U(t)
\right)\otimes\left(\bar U(t)\otimes_{1}\bar
U(t)\ldots\otimes_{(n-1)}\bar U(t) \right)\right)\cdot\xi\left(\vec
x_1,\ldots,\vec x_n,\vec y_1,\ldots,\vec y_n,t\right).
\end{multline}
Therefore, since we obtained before, that equation
\er{vhfffngghkjgghfjjghghghhjghjgghkghggkghghjghSYSPNjjjhhhjjkjkjklkjljljkkhhkjiijjijhbmbgghghdd11}
is equivalent to the primed version of
\er{vhfffngghkjgghfjjghghghhjghjgghkghggkghghjghSYSPNjjjhhhjj} and
at the same time equation
\er{vhfffngghkjgghfjjghghghhjghjgghkghggkghghjghSYSPNjjjhhhjj} is
invariant under the change of non-inertial cartesian coordinate
system, provided we have \er{vyfgjhgjhvhgghSYSPNjj11jhjggj}, with
the help of
\er{vhfffngghkjgghfjjghghghhjghjgghkghggkghghjghSYSPNjjjhhhjjkjkjklkjljljkkhhkjiijjijhbmbgghgh11ggghhgf11gg11},
as before, we deduce
\begin{equation}\label{vhfffngghkjgghfjjghghghhjghjgghkghggkghghjghSYSPNjjjhhhjjkjkjklkjljljkkhhkjiijjijhbmbgghgh11}
i\hbar\frac{\partial}{\partial t}\left(f\circ \hat H_{\vec
u}\right)= \hat H_0\cdot\left(f\circ \hat H_{\vec
u}\right)-\left(f\circ \hat H_{\vec u}\right)\cdot\hat H_0
\end{equation}
in an arbitrary coordinate system. So we get
\er{vhfffngghkjgghfjjghghghhjghjgghkghggkghghjghSYSPNjjjhhhjjkjkjklkjljljkkhhkjiijjijhbmbgghghihihhcc}
or equivalently
\er{vhfffngghkjgghfjjghghghhjghjgghkghggkghghjghSYSPNjjjhhhjjcc}.

Finally, assume that the first and the second particles have the
same mass $m_1=m_2$ and the same charge $\sigma_1=\sigma_2$ and
moreover, assume that we have
$$V\left(\vec x_1,\vec x_2,\vec x_3,\ldots,\vec
x_n,t\right)=V\left(\vec x_2,\vec x_1,\vec x_3,\ldots,\vec
x_n,t\right)\quad\quad\forall\,\vec x_1,\vec x_2,\vec
x_3,\ldots,\vec x_n\in\mathbb{R}^3,\quad\forall t.$$ In this case
clearly by
\er{vhfffngghkjgghfjjghghghhjghjgghkghgghjhggjjkgSYSPNjjjljjkljjkjkjj}
we deduce the following relation:
\begin{multline}\label{fyfgjguyiuiHHgjgjgghbmgjgfhfhhHHljjjhjhjhyyuyuhhh}
A_{1,2}\cdot\psi(\vec x_2,\vec x_1,\vec x_3,\ldots,\vec
x_n,t)=-\psi(\vec x_1,\vec x_2,\vec x_3,\ldots,\vec
x_n,t)\;\;\forall\,\vec x_1,\ldots,\vec
x_n\in\mathbb{R}^3\quad\text{and}\quad \phi=\hat H_{\vec u}\cdot\psi\\
\quad\text{implies}\quad A_{1,2}\cdot\phi(\vec x_2,\vec x_1,\vec
x_3,\ldots,\vec x_n,t)=-\phi(\vec x_1,\vec x_2,\vec x_3,\ldots,\vec
x_n,t)\;\;\forall\,\vec x_1,\ldots,\vec x_n\in\mathbb{R}^3,
\end{multline}
where, as before, by $A_{1,2}:\mathbb{C}^{2^n}\to \mathbb{C}^{2^n}$
we denote the linear operator (matrix) defined as in
\er{fyfgjguyiuiHHgjgjgghhHH} by the following:
\begin{equation}\label{fyfgjguyiuiHHgjgjgghhHHhuhuhuhhkh}
A_{1,2}\cdot\left(a_1\otimes a_2\otimes a_3\otimes\ldots\otimes
a_n\right)=\left(a_2\otimes a_1\otimes a_3\otimes\ldots\otimes
a_n\right)\quad\quad\forall\, a_1,\ldots a_n\in\mathbb{C}^{2}.
\end{equation}
Therefore, by
\er{vhfffngghkjgghfjjghghghhjghjgghkghggkghghjghSYSPNjjjhhhjjkjkjklkjljljkkhhkjiijjihhjjhjhjhjhkkhhknnhkhh11}
and \er{fyfgjguyiuiHHgjgjgghbmgjgfhfhhHHljjjhjhjhyyuyuhhh} we obtain
\begin{multline}\label{fyfgjguyiuiHHgjgjgghbmgjgfhfhhHHljjjhjhjhyyuyuhhhljjjkjj}
A_{1,2}\cdot\psi(\vec x_2,\vec x_1,\vec x_3,\ldots,\vec
x_n,t)=-\psi(\vec x_1,\vec x_2,\vec x_3,\ldots,\vec
x_n,t)\;\;\forall\,\vec x_1,\ldots,\vec
x_n\in\mathbb{R}^3\quad\text{and}\quad \phi=\left(f\circ\hat H_{\vec u}\right)\cdot\psi\\
\quad\text{implies}\quad A_{1,2}\cdot\phi(\vec x_2,\vec x_1,\vec
x_3,\ldots,\vec x_n,t)=-\phi(\vec x_1,\vec x_2,\vec x_3,\ldots,\vec
x_n,t)\;\;\forall\,\vec x_1,\ldots,\vec x_n\in\mathbb{R}^3,
\end{multline}
consistently with
\er{fyfgjguyiuiHHgjgjgghbmgjgfhfhhHHljjjhjhjhyyuyu}.
\end{proof}

\section{Relation between the gravitational and inertial masses and
conservation laws}\label{GravElectro}
\subsection{Basic assumptions and their consequences}
We assumed before that the electromagnetic field is influenced by
the gravitational field. We also can assume that the gravitational
field is influenced by the electromagnetic field. We remind that we
assume the first approximation of the law of gravitation in the form
of
\er{MaxVacFull1ninshtrgravortghhghgjkgghklhjgkghghjjkjhjkkggjkhjkhjjhhfhjhklkhkhjjklzzzyyyNWBWHWPPN222jkgghgg}.
I.e.
\begin{equation}
\label{MaxVacFull1ninshtrgravortghhghgjkgghklhjgkghghjjkjhjkkggjkhjkhjjhhfhjhklkhkhjjklzzzyyyNWBWHWPPN222jkgghggyuhg}
\begin{cases}
curl_{\vec x}\left(curl_{\vec x}\vec v\right)= 0,\\
\frac{\partial}{\partial t}\left(div_{\vec x}\vec v\right)+div_{\vec
x}\left\{\left(div_{\vec x}\vec v\right)\vec
v\right\}+\frac{1}{4}\left|d_{\vec x}\vec v+\{d_{\vec x}\vec
v\}^T\right|^2-\left(div_{\vec x}\vec v\right)^2= -4\pi GM,
\end{cases}
\end{equation}
%
%
%
%
%
%
where $M$ is the density of gravitational masses. However, till now
we said nothing about the relation between the density of inertial
and gravitational masses. If $\mu$ is the density of inertial
masses, then consistently with the classical Newtonian theory of
gravitation we assume that in the absence of essential
electromagnetic fields we should have
\begin{equation}\label{gghjgghfghd}
M=\mu.
\end{equation}
In order to satisfy the laws of conservation of the linear and
angular momentums and energy, consider the following conserved
proper scalar field $Q$, that we call
"electromagnetical-gravitational" mass density, which is negligible
in the absence of electromagnetic fields and satisfies the identity
\begin{equation}
\label{MaxVacFull1ninshtrgravortghhghgjkgghklhjgkghghjjkjhjkkggjkhjkhjjhhfhjhklkhkhjjklzzzyyyNWNWBWHWPPN}
\frac{\partial Q}{\partial t}+div_{\vec x}\left\{Q\vec v\right\}=-
div_\vec x\left\{\frac{1}{4\pi c}\vec D\times \vec B\right\}
\end{equation}
in the general case. Then, instead of \er{gghjgghfghd}, for the
general case of gravitational-electromagnetic fields we consider the
following relation between the gravitational and inertial mass
densities
\begin{equation}\label{gghjgghfghdkjgjj}
M=\mu+Q.
\end{equation}
Then by
\er{MaxVacFull1ninshtrgravortghhghgjkgghklhjgkghghjjkjhjkkggjkhjkhjjhhfhjhklkhkhjjklzzzyyyNWBWHWPPN222jkgghggyuhg}
and \er{gghjgghfghdkjgjj} we have the following law of gravitation:
\begin{equation}
\label{MaxVacFull1ninshtrgravortghhghgjkgghklhjgkghghjjkjhjkkggjkhjkhjjhhfhjhklkhkhjjklzzzyyyNWBWHWPPN}
\begin{cases}
curl_{\vec x}\left(curl_{\vec x}\vec v\right)= 0,\\
\frac{\partial}{\partial t}\left(div_{\vec x}\vec v\right)+div_{\vec
x}\left\{\left(div_{\vec x}\vec v\right)\vec
v\right\}+\frac{1}{4}\left|d_{\vec x}\vec v+\{d_{\vec x}\vec
v\}^T\right|^2-\left(div_{\vec x}\vec v\right)^2=  -4\pi G(\mu+Q).
\end{cases}
\end{equation}
%
%
%
%
%
%
Then as before, we deduce that the laws
\er{MaxVacFull1ninshtrgravortghhghgjkgghklhjgkghghjjkjhjkkggjkhjkhjjhhfhjhklkhkhjjklzzzyyyNWNWBWHWPPN}
and
\er{MaxVacFull1ninshtrgravortghhghgjkgghklhjgkghghjjkjhjkkggjkhjkhjjhhfhjhklkhkhjjklzzzyyyNWBWHWPPN}
are invariant under the change of non-inertial cartesian coordinate
system, provided that $Q'=Q$. We can rewrite
\er{MaxVacFull1ninshtrgravortghhghgjkgghklhjgkghghjjkjhjkkggjkhjkhjjhhfhjhklkhkhjjklzzzyyyNWBWHWPPN}
as
\begin{equation}
\label{MaxVacFull1ninshtrgravortghhghgjkgghklhjgkghghjjkjhjkkggjkhjkhjjhhfhjhklkhkhjjklzzzyyyhjfgfkjgjNWBWHWPPN}
\begin{cases}
curl_{\vec x}\left(curl_{\vec x}\vec v\right)= 0,\\
div_{\vec x}\left\{\frac{\partial\vec v}{\partial t}+d_\vec x\vec
v\cdot\vec v+\frac{1}{2}\vec v\times curl_{\vec x}\vec v\right\}
= -4\pi G(\mu+Q).
\end{cases}
\end{equation}
In particular in the inertial coordinate system $(*)$ we have:
\begin{equation}
\label{MaxVacFull1ninshtrgravortghhghgjkgghklhjgkghghjjkjhjkkggjkhjkhjjhhfhjhklkhkhjjklzzzyyyhjggjhgghhjhNWBWHWPPN}
\begin{cases}
curl_{\vec x}\vec v= 0,\\
div_{\vec x}\left\{\frac{\partial\vec v}{\partial t}+d_\vec x\vec
v\cdot\vec v\right\}= -4\pi G(\mu+Q),
\end{cases}
\end{equation}
that we can rewrite as
\begin{equation}
\label{MaxVacFull1ninshtrgravortghhghgjkgghklhjgkghghjjkjhjkkggjkhjkhjjhhfhjhklkhkhjjklzzzyyyhjggjhgghhjhNWNWBWHWPPN}
\begin{cases}
curl_{\vec x}\vec v= 0,\\
\frac{\partial\vec v}{\partial t}+d_\vec x\vec v\cdot\vec v=
-\nabla_{\vec x}\Phi,
\end{cases}
\end{equation}
where $\Phi$ is the scalar gravitational potential: a proper scalar
field which satisfies in every coordinate system:
\begin{equation}
\label{MaxVacFull1ninshtrgravortghhghgjkgghklhjgkghghjjkjhjkkggjkhjkhjjhhfhjhklkhkhjjklzzzyyyhjggjhgghhjhNWNWNWBWHWPPN}
\Delta_{\vec x}\Phi=4\pi G(\mu+Q).
\end{equation}
Since in the system $(*)$ we have $curl_{\vec x}\vec v=0$ we can
write
\begin{equation}
\label{MaxVacFull1ninshtrgravortghhghgjkgghklhjgkghghjjkjhjkkggjkhjkhjjhhfhjhklkhkhjjklzzzyyyhjggjhgghhjhNWNWNWNWNWBWHWPPN}
\begin{cases}
\vec v=\nabla_{\vec x}Z,\\
\frac{\partial Z}{\partial t}+\frac{1}{2}\left|\nabla_{\vec
x}Z\right|^2=-\Phi.
\end{cases}
\end{equation}
\begin{remark}\label{ghghvghhggghKKK}
Lemma \ref{fgbfghfh} from Appendix gives some insight that the
"electromagnetical-gravitational" mass density $Q$ in
\er{MaxVacFull1ninshtrgravortghhghgjkgghklhjgkghghjjkjhjkkggjkhjkhjjhhfhjhklkhkhjjklzzzyyyNWNWBWHWPPN}
should have the values of the same order as the quantity
$\frac{1}{c^2}\left(|\vec D|^2+|\vec B|^2\right)$ and therefore, in
the usual circumstances is negligible with respect to the inertial
mass density $\mu$. Thus we can write $Q\approx 0$ in
\er{MaxVacFull1ninshtrgravortghhghgjkgghklhjgkghghjjkjhjkkggjkhjkhjjhhfhjhklkhkhjjklzzzyyyNWBWHWPPN},
i.e. the force of gravity in an inertial coordinate system
approximately equals to the classical Newtonian force of gravity.
\end{remark}

\subsection{Conservation of the momentum, angular momentum and
energy} Consider the Maxwell equation in the vacuum in some
cartesian coordinate system $(*)$:
\begin{equation}\label{MaxVacFull1bjkgjhjhgjaaahkjhhjzzzyyykkknnnNWBWHWPPN}
\begin{cases}
curl_{\vec x} \vec H= \frac{4\pi}{c}\vec j+\frac{1}{c}\frac{\partial
\vec D}{\partial t},\\
div_{\vec x} \vec D= 4\pi\rho,\\
curl_{\vec x} \vec E+\frac{1}{c}\frac{\partial \vec B}{\partial t}=0,\\
div_{\vec x} \vec B=0,\\
\vec E=\vec D-\frac{1}{c}\,\vec v\times \vec B,\\
\vec H=\vec B+\frac{1}{c}\,\vec v\times \vec D,
\end{cases}
\end{equation}
and consistently with
\er{noninchgravortbstrjgghguittu2gjgghhjhghjhjgghgghghghtytythvfghfgghjgg},
consider in the system $(*)$ the second Law of Newton for the moving
continuum with the inertial mass density $\mu$ and the field of
velocities $\vec u$:
\begin{equation}\label{MaxVacFull1ninshtrgravortghhghgjkgghklhjgkghghjjkjhjkkggjkhjkhjjhhfhjhkzzzyyykkknnnNWBWHWPPN}
\frac{\partial\vec u}{\partial t}+d_{\vec x}\vec u\cdot \vec
u=-(\vec u-\vec v)\times curl_{\vec x}\vec v +\partial_t\vec
v+d_{\vec x}\vec v\cdot \vec v+\frac{1}{\mu}\tilde{\vec F},
\end{equation}
where $\tilde{\vec F}$ is the total volume density of all
non-gravitational forces acting on the continuum with mass density
$\mu$. Thus, again by \er{apfrm9}, we rewrite
\er{MaxVacFull1ninshtrgravortghhghgjkgghklhjgkghghjjkjhjkkggjkhjkhjjhhfhjhkzzzyyykkknnnNWBWHWPPN}
as:
\begin{multline}\label{MaxVacFull1ninshtrgravortghhghgjkgghklhjgkghghjjkjhjkkggjkhjkhjjhhfhjhkjkhbbgjhzzzyyykkknnnNWBWHWPPN}
\mu\left(\frac{\partial\vec u}{\partial t}+d_{\vec x}\vec u\cdot
\vec u\right)=\frac{\partial(\mu\vec u)}{\partial t}+div_{\vec
x}\left\{\mu\vec u\otimes\vec u\right\}=\\-\mu\vec u\times
curl_{\vec x}\vec v+\mu\partial_{t}\vec
v+\mu\nabla_{\vec x}\left(\frac{1}{2}|\vec v|^2\right)+\rho\vec E+\frac{1}{c}\vec j\times\vec B+\vec F=\\
=-\mu(\vec u-\vec v)\times curl_{\vec x}\vec v +\mu\partial_t\vec v+
d_{\vec x}\vec v\cdot (\mu\vec v)+\rho\vec E+\frac{1}{c}\vec
j\times\vec B+\vec F.
\end{multline}
where $\rho\vec E+\frac{1}{c}\vec j\times\vec B$ is the volume
density of the Lorentz force and $\vec F$ is the total volume
density of all non-gravitational and non-electromagnetic forces
acting on the continuum with mass density $\mu$, which satisfies the
continuum equation:
\begin{equation}\label{MaxVacFull1ninshtrgravortghhghgjkgghklhjgkghghjjkjhjkkggjkhjkhjjhhfhjhkhjjkhjgzzzyyykkknnnNWBWHWPPN}
\frac{\partial\mu}{\partial t}+div_{\vec x}\left(\mu\vec u\right)=0.
\end{equation}
Then, again by \er{apfrm9}, we rewrite
\er{MaxVacFull1ninshtrgravortghhghgjkgghklhjgkghghjjkjhjkkggjkhjkhjjhhfhjhkjkhbbgjhzzzyyykkknnnNWBWHWPPN}
as
\begin{equation}\label{vhfffngghkjgghggtghjgfhjoyuiyuyhiyyukukyihyuSYSPNNWhgjgghyuyy8yuyughghhCCmmGGKK}
\mu\frac{\partial}{\partial t}\left\{\left(\vec u-\vec v\right)
\right\} + \mu \,d_{\vec x}\left\{\left(\vec u-\vec
v\right)\right\}\cdot\vec u +\mu \left\{d_{\vec x}\vec
v\right\}^T\cdot\left(\vec u-\vec v\right)=\rho\vec
E+\frac{1}{c}\vec j\times \vec B+\vec F.
\end{equation}
Thus by
\er{vhfffngghkjgghggtghjgfhjoyuiyuyhiyyukukyihyuSYSPNNWhgjgghyuyy8yuyughghhCCmmGGKK}
and
\er{MaxVacFull1ninshtrgravortghhghgjkgghklhjgkghghjjkjhjkkggjkhjkhjjhhfhjhkhjjkhjgzzzyyykkknnnNWBWHWPPN}
we obtain
\begin{equation}\label{vhfffngghkjgghggtghjgfhjoyuiyuyhiyyukukyihyuSYSPNNWhgjgghyuyy8yuyughghhghghhjjhCCmmGGKK}
\frac{\partial}{\partial t}\left\{ \mu \left(\vec u-\vec v\right)
\right\} + div_{\vec x}\left\{\mu \left(\vec u-\vec
v\right)\otimes\vec u\right\} +\mu \left\{d_{\vec x}\vec
v\right\}^T\cdot\left(\vec u-\vec v\right)=\rho\vec
E+\frac{1}{c}\vec j\times \vec B+\vec F.
\end{equation}
Moreover, multiplying
\er{vhfffngghkjgghggtghjgfhjoyuiyuyhiyyukukyihyuSYSPNNWhgjgghyuyy8yuyughghhCCmmGGKK}
by $(\vec u-\vec v)$ and using
\er{MaxVacFull1ninshtrgravortghhghgjkgghklhjgkghghjjkjhjkkggjkhjkhjjhhfhjhkhjjkhjgzzzyyykkknnnNWBWHWPPN}
we have:
\begin{multline}\label{vhfffngghkjgghggtghjgfhjoyuiyuyhiyyukukyihyuSYSPNNWhgjgghyuyy8yuyughghhyttyytCCmmGGKK}
\frac{\partial}{\partial t}\left\{\frac{\mu}{2}\left|\vec u-\vec
v\right|^2\right\}+div_{\vec x}\left\{\frac{\mu}{2}\left|\vec u-\vec
v\right|^2\vec u\right\}=
\\ -\frac{1}{2}\left(d_{\vec x}\vec v+\left\{d_{\vec x}\vec v\right\}^T\right)^T:\left\{\mu\left(\vec u-\vec v\right)\otimes \left(\vec
u-\vec v\right)\right\}+\vec j\cdot\vec E -\vec v\cdot\left(\rho\vec
E+\frac{1}{c}\vec j\times \vec B\right)+(\vec u-\vec v)\cdot\vec F.
\end{multline}
On the other hand, by Lemma \ref{guigiukhn} and Lemma \ref{fgbfghfh}
in the Appendix we have:
\begin{multline}\label{hvkgkjgkjbjkjjkgjglhhkhjyuyghjhhjhjzzzyyykkknnnNWBWHWPPNNewZZZghghghCCmmGGKK}
\frac{\partial}{\partial t}\left(\frac{1}{4\pi c}\,\vec D\times \vec
B\right)+div_\vec x\left\{\left(\frac{1}{4\pi c}\vec D\times \vec
B\right)\otimes \vec v\right\}=-(d_\vec x \vec
v)^T\cdot\left(\frac{1}{4\pi c}\vec D\times \vec
B\right)\\+\frac{1}{4\pi}div_\vec x\left\{\vec D\otimes \vec D+\vec
B\otimes \vec B-\frac{1}{2}\left(|\vec D|^2+|\vec
B|^2\right)I\right\}-\left(\rho \vec E+\frac{1}{c}\vec j\times \vec
B\right),
\end{multline}
and
\begin{multline}\label{hvkgkjgkjbjbbklnknhihiokhhfjffghvjmbjhjkhlkzzzyyykkknnnNWBWHWPPNNewZZZgghghCCmmGGKK}
\frac{\partial}{\partial t}\left(\frac{|\vec D|^2+|\vec
B|^2}{8\pi}\right)+div_\vec x\left\{\left(\frac{|\vec D|^2+|\vec
B|^2}{8\pi}\right)\vec v\right\}=\\
\frac{1}{4\pi}div_\vec x\left\{(\vec D\otimes \vec D+ \vec B\otimes
\vec B)\cdot \vec v-\frac{1}{2}\left(|\vec D|^2+|\vec
B|^2\right)\vec v-c \vec D\times \vec B\right\}
\\-\left\{\frac{1}{4\pi}\left(div_\vec x\left\{\vec D\otimes \vec D+\vec B\otimes
\vec B-\frac{1}{2}\left(|\vec D|^2+|\vec
B|^2\right)I\right\}\right)-\left(\rho \vec E+\frac{1}{c}\,\vec
j\times \vec B\right)\right\}\cdot \vec v-\vec j\cdot \vec E =\\
-\frac{c}{4\pi}div_\vec x\left\{ \vec D\times \vec
B\right\}+\left(\rho \vec E+\frac{1}{c}\,\vec j\times \vec
B\right)\cdot \vec v-\vec j\cdot \vec E
\\+\frac{1}{8\pi}\left(d_{\vec x}\vec v+\left\{d_{\vec x}\vec
v\right\}^T\right):\left\{\vec D\otimes \vec D+\vec B\otimes \vec
B-\frac{1}{2}\left(|\vec D|^2+|\vec B|^2\right)I\right\}.
\end{multline}
Thus by
\er{vhfffngghkjgghggtghjgfhjoyuiyuyhiyyukukyihyuSYSPNNWhgjgghyuyy8yuyughghhghghhjjhCCmmGGKK}
and
\er{hvkgkjgkjbjkjjkgjglhhkhjyuyghjhhjhjzzzyyykkknnnNWBWHWPPNNewZZZghghghCCmmGGKK}
we have
\begin{multline}\label{vhfffngghkjgghggtghjgfhjoyuiyuyhiyyukukyihyuSYSPNNWhgjgghyuyy8yuyughghhghghhjjhhjhjkhjCCmmGGKK}
\frac{\partial}{\partial t}\left\{ \mu \left(\vec u-\vec v\right)
+\frac{1}{4\pi c}\,\vec D\times \vec B\right\}+ div_{\vec
x}\left\{\left(\mu \left(\vec u-\vec v\right)+\frac{1}{4\pi c}\vec
D\times \vec B\right)\otimes \vec v\right\}
\\+\left\{d_{\vec x}\vec v\right\}^T\cdot\left\{\mu\left(\vec u-\vec v\right)
+\frac{1}{4\pi c}\vec D\times \vec B\right\}=\\
\frac{1}{4\pi}div_\vec x\left\{\vec D\otimes \vec D+\vec B\otimes
\vec B-\frac{1}{2}\left(|\vec D|^2+|\vec B|^2\right)I-4\pi\mu
\left(\vec u-\vec v\right)\otimes(\vec u-\vec v)\right\}+\vec F,
\end{multline}
and by
\er{vhfffngghkjgghggtghjgfhjoyuiyuyhiyyukukyihyuSYSPNNWhgjgghyuyy8yuyughghhyttyytCCmmGGKK}
and
\er{hvkgkjgkjbjbbklnknhihiokhhfjffghvjmbjhjkhlkzzzyyykkknnnNWBWHWPPNNewZZZgghghCCmmGGKK}
we have
\begin{multline}\label{vhfffngghkjgghggtghjgfhjoyuiyuyhiyyukukyihyuSYSPNNWhgjgghyuyy8yuyughghhyttyytgghhgCCmmGGKK}
\frac{\partial}{\partial t}\left\{\frac{\mu}{2}\left|\vec u-\vec
v\right|^2+\frac{|\vec D|^2+|\vec B|^2}{8\pi}\right\}+div_{\vec
x}\left\{\left(\frac{\mu}{2}\left|\vec u-\vec v\right|^2+\frac{|\vec
D|^2+|\vec B|^2}{8\pi}\right)\vec v\right\}\\+div_{\vec
x}\left\{\frac{\mu}{2}\left|\vec u-\vec v\right|^2(\vec u-\vec
v)+\frac{c}{4\pi}\vec D\times \vec B\right\}=
-\frac{1}{2}\left(d_{\vec x}\vec v+\left\{d_{\vec x}\vec
v\right\}^T\right):\left\{\mu\left(\vec u-\vec v\right)\otimes
\left(\vec u-\vec v\right)\right\}\\+\frac{1}{8\pi}\left(d_{\vec
x}\vec v+\left\{d_{\vec x}\vec v\right\}^T\right):\left\{\vec
D\otimes \vec D+\vec B\otimes \vec B-\frac{1}{2}\left(|\vec
D|^2+|\vec B|^2\right)I
\right\}
+(\vec u-\vec v)\cdot\vec F.
\end{multline}
In particular by
\er{vhfffngghkjgghggtghjgfhjoyuiyuyhiyyukukyihyuSYSPNNWhgjgghyuyy8yuyughghhyttyytgghhgCCmmGGKK}
and
\er{vhfffngghkjgghggtghjgfhjoyuiyuyhiyyukukyihyuSYSPNNWhgjgghyuyy8yuyughghhghghhjjhhjhjkhjCCmmGGKK}
we have:
\begin{multline}\label{vhfffngghkjgghggtghjgfhjoyuiyuyhiyyukukyihyuSYSPNNWhgjgghyuyy8yuyughghhyttyytgghhghghgghjhkjhCCmmGGKK}
\frac{\partial}{\partial t}\left\{ \frac{\mu}{2}\left|\vec u-\vec
v\right|^2+\frac{|\vec D|^2+|\vec B|^2}{8\pi}\right\}+div_{\vec
x}\left\{\left(\frac{\mu}{2}\left|\vec u-\vec
v\right|^2+\left(\frac{|\vec D|^2+|\vec
B|^2}{8\pi}\right)\right)\vec v\right\}\\+div_{\vec
x}\left\{\frac{\mu}{2}\left|\vec u-\vec v\right|^2(\vec u-\vec
v)+\frac{c}{4\pi}\vec D\times \vec B\right\}= -\vec
v\cdot\frac{\partial}{\partial t}\left\{ \mu
 \left(\vec
u-\vec v\right) +\frac{1}{4\pi c}\,\vec D\times \vec B\right\} \\-
div_{\vec x}\left\{\left(\left(\mu\left(\vec u-\vec
v\right)+\frac{1}{4\pi c}\vec D\times \vec B\right)\cdot\vec
v\right)\vec v\right\}-div_{\vec x}\left\{\mu \left(\left(\vec
u-\vec v\right)\cdot\vec v\right) \left(\vec u-\vec
v\right)\right\}\\+\frac{1}{4\pi}div_{\vec x}\left\{\left(\vec
D\otimes \vec D+\vec B\otimes \vec B-\frac{1}{2}\left(|\vec
D|^2+|\vec B|^2\right)I\right)\cdot\vec v
\right\}
+\vec u\cdot\vec F,
\end{multline}
and thus
\begin{multline}\label{vhfffngghkjgghggtghjgfhjoyuiyuyhiyyukukyihyuSYSPNNWhgjgghyuyy8yuyughghhyttyytgghhghghgghjhkjhCCmmGGKKuiiuui}
\frac{\partial}{\partial t}\left\{ \frac{\mu}{2}\left|\vec u-\vec
v\right|^2+\frac{|\vec D|^2+|\vec B|^2}{8\pi}+\vec v\cdot\left( \mu
 \left(\vec
u-\vec v\right) +\frac{1}{4\pi c}\,\vec D\times \vec
B\right)\right\}\\+div_{\vec x}\left\{\left(\frac{\mu}{2}\left|\vec
u-\vec v\right|^2+\left(\frac{|\vec D|^2+|\vec
B|^2}{8\pi}\right)\right)\vec v\right\}+div_{\vec
x}\left\{\frac{\mu}{2}\left|\vec u-\vec v\right|^2(\vec u-\vec
v)+\frac{c}{4\pi}\vec D\times \vec B\right\}\\= \frac{\partial\vec
v}{\partial t}\cdot\left( \mu
 \left(\vec
u-\vec v\right) +\frac{1}{4\pi c}\,\vec D\times \vec B\right)-
div_{\vec x}\left\{\left(\left(\mu\left(\vec u-\vec
v\right)+\frac{1}{4\pi c}\vec D\times \vec B\right)\cdot\vec
v\right)\vec v\right\}\\-div_{\vec x}\left\{\mu \left(\left(\vec
u-\vec v\right)\cdot\vec v\right) \left(\vec u-\vec
v\right)\right\}+\frac{1}{4\pi}div_{\vec x}\left\{\left(\vec
D\otimes \vec D+\vec B\otimes \vec B-\frac{1}{2}\left(|\vec
D|^2+|\vec B|^2\right)I\right)\cdot\vec v
\right\}
+\vec u\cdot\vec F.
\end{multline}
Moreover, by
\er{vhfffngghkjgghggtghjgfhjoyuiyuyhiyyukukyihyuSYSPNNWhgjgghyuyy8yuyughghhghghhjjhhjhjkhjCCmmGGKK}
and \er{apfrm9} we have
\begin{multline}\label{vhfffngghkjgghggtghjgfhjoyuiyuyhiyyukukyihyuSYSPNNWhgjgghyuyy8yuyughghhghghhjjhhjhjkhjCCmmGGKKJJKK}
\frac{\partial}{\partial t}\left\{ \mu \left(\vec u-\vec v\right)
+\frac{1}{4\pi c}\,\vec D\times \vec B\right\}+ div_{\vec
x}\left\{\left(\mu \left(\vec u-\vec v\right)+\frac{1}{4\pi c}\vec
D\times \vec B\right)\otimes \vec v+\vec v\otimes\left(\mu\left(\vec
u-\vec v\right)+\frac{1}{4\pi c}\vec D\times \vec B\right)\right\}
\\+\left(\mu\left(\vec u-\vec v\right)
+\frac{1}{4\pi c}\vec D\times \vec B\right)\times curl_{\vec x}\vec
v-\left(div_{\vec x}\left\{\mu \left(\vec u-\vec v\right)
+\frac{1}{4\pi c}\,\vec D\times \vec B\right\}\right)\vec v=\\
\frac{1}{4\pi}div_\vec x\left\{\vec D\otimes \vec D+\vec B\otimes
\vec B-\frac{1}{2}\left(|\vec D|^2+|\vec B|^2\right)I-4\pi\mu
\left(\vec u-\vec v\right)\otimes(\vec u-\vec v)\right\}+\vec F,
\end{multline}
and by
\er{vhfffngghkjgghggtghjgfhjoyuiyuyhiyyukukyihyuSYSPNNWhgjgghyuyy8yuyughghhghghhjjhhjhjkhjCCmmGGKK}
and \er{apfrm6kkk} we have
\begin{multline}\label{vhfffngghkjgghggtghjgfhjoyuiyuyhiyyukukyihyuSYSPNNWhgjgghyuyy8yuyughghhghghhjjhhjhjkhjCCmmGGKKJJKKMod}
\frac{\partial}{\partial t}\left\{ \mu \left(\vec u-\vec v\right)
+\frac{1}{4\pi c}\,\vec D\times \vec B\right\}- curl_{\vec
x}\left\{\vec v\times\left(\mu\left(\vec u-\vec
v\right)+\frac{1}{4\pi c}\vec D\times \vec B\right)\right\}
\\+\left(div_{\vec x}\left\{\mu \left(\vec u-\vec v\right)
+\frac{1}{4\pi c}\,\vec D\times \vec B\right\}\right)\vec
v+\left(d_{\vec x}\vec v+\left\{d_{\vec x}\vec
v\right\}^T\right)\cdot\left\{\mu\left(\vec u-\vec v\right)
+\frac{1}{4\pi c}\vec D\times \vec B\right\}=\\
\frac{1}{4\pi}div_\vec x\left\{\vec D\otimes \vec D+\vec B\otimes
\vec B-\frac{1}{2}\left(|\vec D|^2+|\vec B|^2\right)I-4\pi\mu
\left(\vec u-\vec v\right)\otimes(\vec u-\vec v)\right\}+\vec F.
\end{multline}
On the other hand for every vector fields $\Gamma:\R^3\to\R^3$ and
$\Lambda:\R^3\to\R^3$ and every scalar field $P:\R^3\to\R$ we have:
\begin{multline}\label{yguiyiu7yytyutkhjzzzyyykkknnnNWBWHWPPNKK}
\vec x\times div_{\vec
x}\{\Gamma\otimes\Lambda+\Lambda\otimes\Gamma\}=div_{\vec
x}\left\{(\vec x\times\Gamma)\otimes\Lambda+(\vec
x\times\Lambda)\otimes\Gamma\right\},\\
\vec x\times div_{\vec x}\{P\Gamma\otimes\Gamma\}=div_{\vec
x}\left\{P(\vec
x\times\Gamma)\otimes\Gamma\right\}\quad\text{and}\quad\vec
x\times\nabla_{\vec x}P=-curl_{\vec x}\{P\,\vec x\}.
\end{multline}
Thus inserting \er{yguiyiu7yytyutkhjzzzyyykkknnnNWBWHWPPNKK} into
\er{vhfffngghkjgghggtghjgfhjoyuiyuyhiyyukukyihyuSYSPNNWhgjgghyuyy8yuyughghhghghhjjhhjhjkhjCCmmGGKKJJKK}
we obtain:
\begin{multline}\label{hvkgkjgkjbjkjjkgjglhhkhjyuyghjhhjhjfghfdhgdhdfdhzzzyyyjffjjkkhhkgjgkjhfhjhgfffjgjhgjffgjggjgjggjhgkkkggjgjgmomnnnNWKK}
\frac{\partial}{\partial t}\left\{\vec x\times\left(\mu(\vec u-\vec
v)+\frac{1}{4\pi c}\,\vec D\times \vec B\right)\right\}+div_\vec
x\left\{\mu\left(\vec x\times(\vec u-\vec v)\right)\otimes(\vec
u-\vec v)\right\}\\+div_{\vec x}\left\{\left(\vec x \times\vec
v\right)\otimes\left(\mu(\vec u-\vec v)+\frac{1}{4\pi c}\vec D\times
\vec B\right)+\left(\vec x\times\left(\mu(\vec u-\vec
v)+\frac{1}{4\pi c}\vec D\times \vec B\right)\right)\otimes \vec
v\right\}\\-\vec x\times\left\{\left(div_{\vec x}\left\{\mu(\vec
u-\vec v)+\frac{1}{4\pi c}\vec D\times \vec B\right\}\right)\vec
v-\left(\mu(\vec u-\vec v)+\frac{1}{4\pi c}\vec
D\times \vec B\right)\times curl_{\vec x}\vec v\right\}\\
=\frac{1}{4\pi}div_\vec x\left\{(\vec x\times \vec D)\otimes \vec
D+(\vec x\times\vec B)\otimes \vec B\right\}+curl_{\vec
x}\left\{\frac{1}{2}\left(|\vec D|^2+|\vec B|^2\right)\vec
x\right\}+\vec x\times\vec F.
\end{multline}

 Next assume that the system $(*)$ is inertial. Then, since by \er{MaxVacFull1ninshtrgravortghhghgjkgghklhjgkghghjjkjhjkkggjkhjkhjjhhfhjhklkhkhjjklzzzyyyNWNWBWHWPPN}
and
\er{MaxVacFull1ninshtrgravortghhghgjkgghklhjgkghghjjkjhjkkggjkhjkhjjhhfhjhkhjjkhjgzzzyyykkknnnNWBWHWPPN}
we have:
\begin{equation}
\label{MaxVacFull1ninshtrgravortghhghgjkgghklhjgkghghjjkjhjkkggjkhjkhjjhhfhjhklkhkhjjklzzzyyyNWNWBWHWPPNKKKK}
\frac{\partial}{\partial t}\left(\mu+Q\right)+div_{\vec
x}\left\{\left(\mu+Q\right)\vec v\right\}=- div_\vec
x\left\{\mu(\vec u-\vec v)+\frac{1}{4\pi c}\vec D\times \vec
B\right\},
\end{equation}
by
\er{MaxVacFull1ninshtrgravortghhghgjkgghklhjgkghghjjkjhjkkggjkhjkhjjhhfhjhklkhkhjjklzzzyyyNWNWBWHWPPNKKKK},
\er{MaxVacFull1ninshtrgravortghhghgjkgghklhjgkghghjjkjhjkkggjkhjkhjjhhfhjhklkhkhjjklzzzyyyhjggjhgghhjhNWNWBWHWPPN}
and
\er{MaxVacFull1ninshtrgravortghhghgjkgghklhjgkghghjjkjhjkkggjkhjkhjjhhfhjhklkhkhjjklzzzyyyhjggjhgghhjhNWNWNWBWHWPPN}
we have
\begin{multline}\label{vhfffngghkjgghggtghjgfhjoyuiyuyhiyyukukyihyuSYSPNNWhgjgghyuyy8yuyughghhghghhjjhhjhjkhjCCmmGGKKJJKKjkhjh}
-\left(div_{\vec x}\left\{\mu \left(\vec u-\vec v\right)
+\frac{1}{4\pi c}\,\vec D\times \vec B\right\}\right)\vec
v=\left(\frac{\partial}{\partial t}\left(\mu+Q\right)\right)\vec
v+\left(div_{\vec x}\left\{\left(\mu+Q\right)\vec
v\right\}\right)\vec v=\\
\frac{\partial}{\partial t}\left(\left(\mu+Q\right)\vec
v\right)+div_{\vec x}\left\{\left(\mu+Q\right)\vec v\otimes\vec
v\right\}-\left(\mu+Q\right)\left(\frac{\partial\vec v}{\partial
t}+d_{\vec x}\vec v\cdot\vec v\right)=\\ \frac{\partial}{\partial
t}\left(\left(\mu+Q\right)\vec v\right)+div_{\vec
x}\left\{\left(\mu+Q\right)\vec v\otimes\vec
v\right\}+\left(\mu+Q\right)\nabla_{\vec x}\Phi=\\
\frac{\partial}{\partial t}\left(\left(\mu+Q\right)\vec
v\right)+div_{\vec x}\left\{\left(\mu+Q\right)\vec v\otimes\vec
v\right\}+\frac{1}{4\pi G}\left(\Delta_{\vec
x}\Phi\right)\nabla_{\vec x}\Phi=\\
\frac{\partial}{\partial t}\left(\left(\mu+Q\right)\vec
v\right)+div_{\vec x}\left\{\left(\mu+Q\right)\vec v\otimes\vec
v\right\}+\frac{1}{4\pi G}div_{\vec x}\left\{\nabla_{\vec
x}\Phi\otimes \nabla_{\vec x}\Phi-\frac{1}{2}\left|\nabla_{\vec
x}\Phi\right|^2I\right\}.
\end{multline}
Moreover, by \er{yguiyiu7yytyutkhjzzzyyykkknnnNWBWHWPPNKK} and
\er{vhfffngghkjgghggtghjgfhjoyuiyuyhiyyukukyihyuSYSPNNWhgjgghyuyy8yuyughghhghghhjjhhjhjkhjCCmmGGKKJJKKjkhjh}
we have
\begin{multline}\label{vhfffngghkjgghggtghjgfhjoyuiyuyhiyyukukyihyuSYSPNNWhgjgghyuyy8yuyughghhghghhjjhhjhjkhjCCmmGGKKJJKKjkhjhyyyuu}
-\left(div_{\vec x}\left\{\mu \left(\vec u-\vec v\right)
+\frac{1}{4\pi c}\,\vec D\times \vec B\right\}\right)\vec
x\times\vec v= \frac{\partial}{\partial
t}\left(\left(\mu+Q\right)\vec x\times\vec v\right)+div_{\vec
x}\left\{\left(\mu+Q\right)\left(\vec x\times\vec
v\right)\otimes\vec v\right\}\\+\frac{1}{4\pi G}div_{\vec
x}\left\{\left(\vec x\times\nabla_{\vec x}\Phi\right)\otimes
\nabla_{\vec x}\Phi\right\}+\frac{1}{8\pi G}curl_{\vec
x}\left\{\left|\nabla_{\vec x}\Phi\right|^2\vec x\right\}.
\end{multline}
Therefore, by inserting
\er{vhfffngghkjgghggtghjgfhjoyuiyuyhiyyukukyihyuSYSPNNWhgjgghyuyy8yuyughghhghghhjjhhjhjkhjCCmmGGKKJJKKjkhjh}
into
\er{vhfffngghkjgghggtghjgfhjoyuiyuyhiyyukukyihyuSYSPNNWhgjgghyuyy8yuyughghhghghhjjhhjhjkhjCCmmGGKKJJKK}
and using
\er{MaxVacFull1ninshtrgravortghhghgjkgghklhjgkghghjjkjhjkkggjkhjkhjjhhfhjhklkhkhjjklzzzyyyhjggjhgghhjhNWNWBWHWPPN}
we deduce the following conservation of the momentum:
\begin{multline}\label{vhfffngghkjgghggtghjgfhjoyuiyuyhiyyukukyihyuSYSPNNWhgjgghyuyy8yuyughghhghghhjjhhjhjkhjCCmmGGKKJJKKhghgghghgh}
\frac{\partial}{\partial t}\left\{ \mu\vec u+Q\vec v +\frac{1}{4\pi
c}\,\vec D\times \vec B\right\}+ div_{\vec x}\left\{\mu\vec
u\otimes\vec u+Q\vec v\otimes\vec v+\left(\frac{1}{4\pi c}\vec
D\times \vec B\right)\otimes \vec v+\vec v\otimes\left(\frac{1}{4\pi
c}\vec D\times \vec B\right)\right\}
\\+\frac{1}{4\pi G}div_{\vec x}\left\{\nabla_{\vec
x}\Phi\otimes \nabla_{\vec x}\Phi-\frac{1}{2}\left|\nabla_{\vec
x}\Phi\right|^2I\right\}-div_\vec x\left\{\mu
\left(\vec u-\vec v\right)\otimes(\vec u-\vec v)\right\}=\\
\frac{\partial}{\partial t}\left\{ \mu \left(\vec u-\vec v\right)
+\frac{1}{4\pi c}\,\vec D\times \vec B\right\}+ div_{\vec
x}\left\{\left(\mu \left(\vec u-\vec v\right)+\frac{1}{4\pi c}\vec
D\times \vec B\right)\otimes \vec v+\vec v\otimes\left(\mu\left(\vec
u-\vec v\right)+\frac{1}{4\pi c}\vec D\times \vec B\right)\right\}
\\+\frac{\partial}{\partial t}\left(\left(\mu+Q\right)\vec
v\right)+div_{\vec x}\left\{\left(\mu+Q\right)\vec v\otimes\vec
v\right\}+\frac{1}{4\pi G}div_{\vec x}\left\{\nabla_{\vec
x}\Phi\otimes \nabla_{\vec x}\Phi-\frac{1}{2}\left|\nabla_{\vec
x}\Phi\right|^2I\right\}=\\
\frac{1}{4\pi}div_\vec x\left\{\vec D\otimes \vec D+\vec B\otimes
\vec B-\frac{1}{2}\left(|\vec D|^2+|\vec B|^2\right)I-4\pi\mu
\left(\vec u-\vec v\right)\otimes(\vec u-\vec v)\right\}+\vec F,
\end{multline}
and by inserting
\er{vhfffngghkjgghggtghjgfhjoyuiyuyhiyyukukyihyuSYSPNNWhgjgghyuyy8yuyughghhghghhjjhhjhjkhjCCmmGGKKJJKKjkhjhyyyuu}
into
\er{hvkgkjgkjbjkjjkgjglhhkhjyuyghjhhjhjfghfdhgdhdfdhzzzyyyjffjjkkhhkgjgkjhfhjhgfffjgjhgjffgjggjgjggjhgkkkggjgjgmomnnnNWKK}
and using
\er{MaxVacFull1ninshtrgravortghhghgjkgghklhjgkghghjjkjhjkkggjkhjkhjjhhfhjhklkhkhjjklzzzyyyhjggjhgghhjhNWNWBWHWPPN}
we deduce the following conservation of the angular momentum:
\begin{multline}\label{hvkgkjgkjbjkjjkgjglhhkhjyuyghjhhjhjfghfdhgdhdfdhzzzyyyjffjjkkhhkgjgkjhfhjhgfffjgjhgjffgjggjgjggjhgkkkggjgjgmomnnnNWKKhgjggh}
\frac{\partial}{\partial t}\left\{\vec x\times\left(\mu\vec u+Q\vec
v+\frac{1}{4\pi c}\,\vec D\times \vec B\right)\right\}+\frac{1}{4\pi
G}div_{\vec x}\left\{\left(\vec x\times\nabla_{\vec
x}\Phi\right)\otimes \nabla_{\vec x}\Phi\right\}+\frac{1}{8\pi
G}curl_{\vec x}\left\{\left|\nabla_{\vec x}\Phi\right|^2\vec
x\right\}\\+div_{\vec x}\left\{\mu\left(\vec x\times\vec
u\right)\otimes\vec u+Q\left(\vec x\times\vec v\right)\otimes\vec
v+\left(\vec x \times\vec v\right)\otimes\left(\frac{1}{4\pi c}\vec
D\times \vec B\right)+\left(\vec x\times\left(\frac{1}{4\pi c}\vec
D\times \vec
B\right)\right)\otimes \vec v\right\}=\\
\frac{\partial}{\partial t}\left\{\vec x\times\left(\mu(\vec u-\vec
v)+\frac{1}{4\pi c}\,\vec D\times \vec B\right)\right\}+div_\vec
x\left\{\mu\left(\vec x\times(\vec u-\vec v)\right)\otimes(\vec
u-\vec v)\right\}\\+div_{\vec x}\left\{\left(\vec x \times\vec
v\right)\otimes\left(\mu(\vec u-\vec v)+\frac{1}{4\pi c}\vec D\times
\vec B\right)+\left(\vec x\times\left(\mu(\vec u-\vec
v)+\frac{1}{4\pi c}\vec D\times \vec B\right)\right)\otimes \vec
v\right\}\\+\frac{\partial}{\partial t}\left(\left(\mu+Q\right)\vec
x\times\vec v\right)+div_{\vec x}\left\{\left(\mu+Q\right)\left(\vec
x\times\vec v\right)\otimes\vec v\right\}\\+\frac{1}{4\pi
G}div_{\vec x}\left\{\left(\vec x\times\nabla_{\vec
x}\Phi\right)\otimes \nabla_{\vec x}\Phi\right\}+\frac{1}{8\pi
G}curl_{\vec
x}\left\{\left|\nabla_{\vec x}\Phi\right|^2\vec x\right\}\\
=\frac{1}{4\pi}div_\vec x\left\{(\vec x\times \vec D)\otimes \vec
D+(\vec x\times\vec B)\otimes \vec B\right\}+curl_{\vec
x}\left\{\frac{1}{2}\left(|\vec D|^2+|\vec B|^2\right)\vec
x\right\}+\vec x\times\vec F.
\end{multline}
Finally, by
\er{MaxVacFull1ninshtrgravortghhghgjkgghklhjgkghghjjkjhjkkggjkhjkhjjhhfhjhklkhkhjjklzzzyyyNWNWBWHWPPNKKKK},
\er{MaxVacFull1ninshtrgravortghhghgjkgghklhjgkghghjjkjhjkkggjkhjkhjjhhfhjhklkhkhjjklzzzyyyhjggjhgghhjhNWNWBWHWPPN}
and
\er{MaxVacFull1ninshtrgravortghhghgjkgghklhjgkghghjjkjhjkkggjkhjkhjjhhfhjhklkhkhjjklzzzyyyhjggjhgghhjhNWNWNWBWHWPPN}
we have
\begin{multline}\label{vhfffngghkjgghggtghjgfhjoyuiyuyhiyyukukyihyuSYSPNNWhgjgghyuyy8yuyughghhyttyytgghhghghgghjhkjhCCmmGGKKuiiuuihjgghghhhjhjhjhj}
\frac{\partial\vec v}{\partial t}\cdot\left( \mu
 \left(\vec
u-\vec v\right) +\frac{1}{4\pi c}\,\vec D\times \vec
B\right)=-\nabla_{\vec x}\left(\Phi+\frac{1}{2}\left|\vec
v\right|^2\right)\cdot \left( \mu
 \left(\vec
u-\vec v\right) +\frac{1}{4\pi c}\,\vec D\times \vec B\right)=\\
\left(\Phi+\frac{1}{2}\left|\vec v\right|^2\right) \left(div_{\vec
x}\left\{ \mu
 \left(\vec
u-\vec v\right) +\frac{1}{4\pi c}\,\vec D\times \vec
B\right\}\right)-div_{\vec x}\left\{\left(\Phi+\frac{1}{2}\left|\vec
v\right|^2\right)\left( \mu
 \left(\vec
u-\vec v\right) +\frac{1}{4\pi c}\,\vec D\times \vec
B\right)\right\}=\\
-\left(\Phi+\frac{1}{2}\left|\vec v\right|^2\right)
\left(\frac{\partial}{\partial t}\left(\mu+Q\right)+div_{\vec
x}\left\{\left(\mu+Q\right)\vec v\right\}\right)-div_{\vec
x}\left\{\left(\Phi+\frac{1}{2}\left|\vec v\right|^2\right)\left(
\mu
 \left(\vec
u-\vec v\right) +\frac{1}{4\pi c}\,\vec D\times \vec
B\right)\right\}\\= -\left(\Phi+\frac{1}{2}\left|\vec
v\right|^2\right) \left(\frac{\partial}{\partial
t}\left(\mu+Q\right)\right)+\left(\mu+Q\right)\nabla_{\vec
x}\left(\Phi+\frac{1}{2}\left|\vec v\right|^2\right)\cdot\vec v-
div_{\vec x}\left\{\left(\Phi+\frac{1}{2}\left|\vec
v\right|^2\right)\left(\mu+Q\right)\vec v\right\}\\-div_{\vec
x}\left\{\left(\Phi+\frac{1}{2}\left|\vec v\right|^2\right)\left(
\mu
 \left(\vec
u-\vec v\right) +\frac{1}{4\pi c}\,\vec D\times \vec
B\right)\right\}=-\frac{1}{4\pi G}\Phi
\left(\frac{\partial}{\partial
t}\left(\Delta_{\vec x}\Phi\right)\right)\\
-\left(\frac{1}{2}\left|\vec v\right|^2\right)
\left(\frac{\partial}{\partial
t}\left(\mu+Q\right)\right)-\left(\mu+Q\right)\frac{\partial\vec
v}{\partial t}\cdot\vec v- div_{\vec
x}\left\{\left(\Phi+\frac{1}{2}\left|\vec
v\right|^2\right)\left(\mu+Q\right)\vec v\right\}\\-div_{\vec
x}\left\{\left(\Phi+\frac{1}{2}\left|\vec v\right|^2\right)\left(
\mu
 \left(\vec
u-\vec v\right) +\frac{1}{4\pi c}\,\vec D\times \vec
B\right)\right\}=\frac{1}{8\pi G} \frac{\partial}{\partial
t}\left(\left|\nabla_{\vec x}\Phi\right|^2\right) -
\frac{\partial}{\partial
t}\left(\frac{1}{2}\left(\mu+Q\right)\left|\vec
v\right|^2\right)\\-\frac{1}{4\pi G}div_{\vec x}\left\{\Phi
\frac{\partial}{\partial t}\left(\nabla_{\vec
x}\Phi\right)\right\}-div_{\vec
x}\left\{\left(\Phi+\frac{1}{2}\left|\vec v\right|^2\right)\left(
\mu \vec u+Q\vec v+\frac{1}{4\pi c}\,\vec D\times \vec
B\right)\right\}
\end{multline}
Then by inserting
\er{vhfffngghkjgghggtghjgfhjoyuiyuyhiyyukukyihyuSYSPNNWhgjgghyuyy8yuyughghhyttyytgghhghghgghjhkjhCCmmGGKKuiiuuihjgghghhhjhjhjhj}
into
\er{vhfffngghkjgghggtghjgfhjoyuiyuyhiyyukukyihyuSYSPNNWhgjgghyuyy8yuyughghhyttyytgghhghghgghjhkjhCCmmGGKKuiiuui}
we deuce:
\begin{multline}\label{vhfffngghkjgghggtghjgfhjoyuiyuyhiyyukukyihyuSYSPNNWhgjgghyuyy8yuyughghhyttyytgghhghghgghjhkjhCCmmGGKKuiiuuihjgghgh}
\frac{\partial}{\partial t}\left\{ \frac{\mu}{2}\left|\vec u-\vec
v\right|^2+\frac{|\vec D|^2+|\vec B|^2}{8\pi}+\vec v\cdot\left( \mu
 \left(\vec
u-\vec v\right) +\frac{1}{4\pi c}\,\vec D\times \vec
B\right)\right\}\\+div_{\vec x}\left\{\left(\frac{\mu}{2}\left|\vec
u-\vec v\right|^2+\left(\frac{|\vec D|^2+|\vec
B|^2}{8\pi}\right)\right)\vec v\right\}+div_{\vec
x}\left\{\frac{\mu}{2}\left|\vec u-\vec v\right|^2(\vec u-\vec
v)+\frac{c}{4\pi}\vec D\times \vec B\right\}\\= \frac{1}{8\pi G}
\frac{\partial}{\partial t}\left(\left|\nabla_{\vec
x}\Phi\right|^2\right) - \frac{\partial}{\partial
t}\left(\frac{1}{2}\left(\mu+Q\right)\left|\vec
v\right|^2\right)\\-\frac{1}{4\pi G}div_{\vec x}\left\{\Phi
\frac{\partial}{\partial t}\left(\nabla_{\vec
x}\Phi\right)\right\}-div_{\vec
x}\left\{\left(\Phi+\frac{1}{2}\left|\vec v\right|^2\right)\left(
\mu \vec u+Q\vec v+\frac{1}{4\pi c}\,\vec D\times \vec
B\right)\right\}\\- div_{\vec x}\left\{\left(\left(\mu\left(\vec
u-\vec v\right)+\frac{1}{4\pi c}\vec D\times \vec B\right)\cdot\vec
v\right)\vec v\right\}\\-div_{\vec x}\left\{\mu \left(\left(\vec
u-\vec v\right)\cdot\vec v\right) \left(\vec u-\vec
v\right)\right\}+\frac{1}{4\pi}div_{\vec x}\left\{\left(\vec
D\otimes \vec D+\vec B\otimes \vec B-\frac{1}{2}\left(|\vec
D|^2+|\vec B|^2\right)I\right)\cdot\vec v
\right\}
+\vec u\cdot\vec F.
\end{multline}
Then, using \er{apfrm1} and the last two equalities in
\er{MaxVacFull1bjkgjhjhgjaaahkjhhjzzzyyykkknnnNWBWHWPPN},  we
rewrite
\er{vhfffngghkjgghggtghjgfhjoyuiyuyhiyyukukyihyuSYSPNNWhgjgghyuyy8yuyughghhyttyytgghhghghgghjhkjhCCmmGGKKuiiuuihjgghgh}
in the form of the following conservation of the energy:
\begin{multline}\label{vhfffngghkjgghggtghjgfhjoyuiyuyhiyyukukyihyuSYSPNNWhgjgghyuyy8yuyughghhyttyytgghhghghgghjhkjhCCmmGGKKuiiuuihjgghghiuiuiui}
\frac{\partial}{\partial t}\left\{ \frac{\mu}{2}\left|\vec
u\right|^2+\frac{Q}{2}\left|\vec v\right|^2+\frac{\vec D\cdot\vec
E+\vec B\cdot\vec H}{8\pi}-\frac{1}{8\pi G} \left|\nabla_{\vec
x}\Phi\right|^2\right\}\\+div_{\vec x}\left\{\frac{\mu}{2}\left|\vec
u\right|^2\vec u+\frac{Q}{2}\left|\vec v\right|^2\vec
v+\left(\frac{\vec D\cdot\vec E+\vec B\cdot\vec H}{8\pi}\right)\vec
v+\frac{1}{8\pi c}\left|\vec v\right|^2\left(\vec D\times \vec
B\right)\right\}+div_{\vec x}\left\{\frac{c}{4\pi}\vec D\times \vec
B\right\}\\=-\frac{1}{4\pi G}div_{\vec x}\left\{\Phi
\frac{\partial}{\partial t}\left(\nabla_{\vec
x}\Phi\right)\right\}-div_{\vec x}\left\{\Phi\left( \mu \vec u+Q\vec
v+\frac{1}{4\pi c}\,\vec D\times \vec
B\right)\right\}\\+\frac{1}{4\pi}div_{\vec x}\left\{\left(\vec
D\otimes \vec D+\vec B\otimes \vec B-\frac{1}{2}\left(|\vec
D|^2+|\vec B|^2\right)I\right)\cdot\vec v
\right\}
+\vec u\cdot\vec F.
\end{multline}
As a consequence of
\er{vhfffngghkjgghggtghjgfhjoyuiyuyhiyyukukyihyuSYSPNNWhgjgghyuyy8yuyughghhghghhjjhhjhjkhjCCmmGGKKJJKKhghgghghgh},
\er{hvkgkjgkjbjkjjkgjglhhkhjyuyghjhhjhjfghfdhgdhdfdhzzzyyyjffjjkkhhkgjgkjhfhjhgfffjgjhgjffgjggjgjggjhgkkkggjgjgmomnnnNWKKhgjggh}
and
\er{vhfffngghkjgghggtghjgfhjoyuiyuyhiyyukukyihyuSYSPNNWhgjgghyuyy8yuyughghhyttyytgghhghghgghjhkjhCCmmGGKKuiiuuihjgghghiuiuiui}
we infer that we have the following proposition:
\begin{proposition}
Consider the Maxwell equation for the vacuum in the form
\er{MaxVacFull1bjkgjhjhgjaaahkjhhjzzzyyykkknnnNWBWHWPPN} and the
second Law of Newton for the moving continuum in the form
\er{MaxVacFull1ninshtrgravortghhghgjkgghklhjgkghghjjkjhjkkggjkhjkhjjhhfhjhkjkhbbgjhzzzyyykkknnnNWBWHWPPN}.
Next, assume that in some cartesian coordinate system $(*)$ we
observe the gravitational law in the form of
\er{MaxVacFull1ninshtrgravortghhghgjkgghklhjgkghghjjkjhjkkggjkhjkhjjhhfhjhklkhkhjjklzzzyyyhjggjhgghhjhNWNWBWHWPPN},
\er{MaxVacFull1ninshtrgravortghhghgjkgghklhjgkghghjjkjhjkkggjkhjkhjjhhfhjhklkhkhjjklzzzyyyhjggjhgghhjhNWNWNWBWHWPPN}
and
\er{MaxVacFull1ninshtrgravortghhghgjkgghklhjgkghghjjkjhjkkggjkhjkhjjhhfhjhklkhkhjjklzzzyyyNWNWBWHWPPN}.
Then in the system $(*)$ we have the following conservation laws of
the linear and angular momentums and energy:
\begin{multline}\label{hvkgkjgkjbjkjjkgjglhhkhjyuyghjhhjhjfghfdhgdhdfdhzzzyyyjffjjkkhhkgjgkjhfhjhgfffjgjhgjffgjggjgjggjhgkkkggjgjgnnnmmmNWNWNWNWBWHWMPPN}
\frac{\partial}{\partial t}\left(\mu\vec u+Q\vec v+\frac{1}{4\pi
c}\,\vec D\times \vec B\right)=\\-div_\vec x\left\{\mu\vec
u\otimes\vec u+Q\vec v\otimes\vec v+\left(\frac{1}{4\pi c}\vec
D\times \vec B\right)\otimes \vec v+\vec v\otimes\left(\frac{1}{4\pi
c}\vec D\times \vec
B\right)\right\}\\
+\frac{1}{4\pi}div_\vec x\left\{\vec D\otimes \vec D+\vec B\otimes
\vec B-\frac{1}{2}\left(|\vec D|^2+|\vec
B|^2\right)I-\frac{1}{G}\nabla_{\vec x}\Phi\otimes\nabla_{\vec
x}\Phi+\frac{1}{2G}\left|\nabla_{\vec x}\Phi\right|^2 I\right\}+\vec
F,
\end{multline}
\begin{multline}\label{hvkgkjgkjbjkjjkgjglhhkhjyuyghjhhjhjfghfdhgdhdfdhzzzyyyjffjjkkhhkgjgkjhfhjhgfffjgjhgjffgjggjgjggjhgkkkggjgjgmomnnnmmmNWNWBWHWMMPPN}
\frac{\partial}{\partial t}\left(\vec x\times(\mu\vec u)+\vec
x\times(Q\vec v)+\vec x\times\left(\frac{1}{4\pi c}\,\vec D\times
\vec B\right)\right)=\\-div_\vec x\left\{\mu(\vec x\times\vec
u)\otimes\vec u+Q(\vec x\times\vec v)\otimes\vec v+\left(\vec
x\times\left(\frac{1}{4\pi c}\vec D\times \vec
B\right)\right)\otimes \vec v+(\vec x\times\vec
v)\otimes\left(\frac{1}{4\pi c}\vec D\times \vec
B\right)\right\}\\
+\frac{1}{4\pi}div_\vec x\left\{(\vec x\times\vec D)\otimes \vec
D+(\vec x\times\vec B)\otimes \vec B-\frac{1}{G}(\vec
x\times\nabla_{\vec x}\Phi)\otimes\nabla_{\vec
x}\Phi\right\}\\+\frac{1}{8\pi}curl_{\vec x}\left\{\left(|\vec
D|^2+|\vec B|^2-\frac{1}{G}\left|\nabla_{\vec
x}\Phi\right|^2\right)\vec x\right\}+ \vec x\times \vec F,
\end{multline}
and
\begin{multline}\label{MaxVacFull1ninshtrgravortghhghgjkgghklhjgkghghjjkjhjkkggjkhjkhjjhhfhjhkjkhbbgjhhkjhhklhzzzyyykkkgkhjjhgfhjjffghuikkgkjghhjkjknnnmmmNWNWBWHWhhjhhENPPN}
\frac{\partial}{\partial t}\left(\frac{\mu|\vec
u|^2}{2}+\frac{Q}{2}\left|\vec v\right|^2+\frac{\vec D\cdot\vec
E+\vec B\cdot\vec H}{8\pi}-\frac{1}{8\pi G}\big|\nabla_{\vec
x}\Phi\big|^2\right)=\\
=-div_\vec x\left\{\left(\frac{\mu|\vec u|^2}{2}\right)\vec
u+\left(\frac{Q|\vec v|^2}{2}\right)\vec v+\frac{1}{2}\left|\vec
v\right|^2\left(\frac{1}{4\pi c}\,\vec D\times \vec
B\right)+\left(\frac{\vec D\cdot\vec E+\vec B\cdot\vec
H}{8\pi}\right)\vec v\right\}
\\
+\frac{1}{4\pi}div_\vec x\left\{(\vec D\otimes \vec D+ \vec B\otimes
\vec B)\cdot \vec v-\frac{1}{2}\left(|\vec D|^2+|\vec
B|^2\right)\vec v-c \vec D\times \vec B\right\}\\-div_\vec
x\left\{\Phi\left(\mu\vec u+Q\vec v+\frac{1}{4\pi c}\,\vec D\times
\vec B\right)\right\}
-\frac{1}{4\pi G}div_{\vec x}\left\{\Phi\frac{\partial}{\partial
t}(\nabla_{\vec x}\Phi)\right\}+\vec F\cdot\vec u.
\end{multline}
%
%
%
%
%
%
\end{proposition}

\section{Lagrangian of the unified Gravitational-Electromagnetic field}\label{ghjfhgfdhd} Given known the distribution of inertial mass
density of some continuum medium $\mu:=\mu(\vec x,t)$, the field of
velocities of this medium $\vec u:=\vec u(\vec x,t)$, the charge
density $\rho:=\rho(\vec x,t)$ and the current density $\vec j:=\vec
j(\vec x,t)$ consider a Lagrangian density $L$ defined by
\begin{multline}\label{vhfffngghkjgghDD}
L\left(\vec A,\Psi,\vec v,\Phi,\vec p,\vec
x,t\right):=\frac{1}{8\pi}\left|-\nabla_{\vec
x}\Psi-\frac{1}{c}\frac{\partial\vec A}{\partial t}+\frac{1}{c}\vec
v\times curl_{\vec x}\vec A\right|^2-\frac{1}{8\pi}\left|curl_{\vec
x}\vec A\right|^2-\left(\rho\Psi-\frac{1}{c}\vec A\cdot\vec
j\right)\\+\frac{\mu}{2}\left|\vec u-\vec v\right|^2+
\frac{1}{2}\left(d_{\vec x}\vec v+\left\{d_{\vec x}\vec
v\right\}^T\right):\left(d_{\vec x}\vec p+\left\{d_{\vec x}\vec
p\right\}^T\right)-2\left(div_{\vec x}\vec v\right)\left(div_{\vec
x}\vec p\right)\\+\frac{1}{4\pi G}\left(div_{\vec x}\vec
v\right)\left(\frac{\partial \Phi}{\partial t}+\vec
v\cdot\nabla_{\vec x} \Phi\right)+\frac{1}{4\pi
G}\Phi\left(div_{\vec x}\vec v\right)^2-\frac{\Phi}{16\pi
G}\left|d_{\vec x}\vec v+\left\{d_{\vec x}\vec
v\right\}^T\right|^2+\frac{1}{8\pi G}\left|\nabla_{\vec
x}\Phi\right|^2,
\end{multline}
where $\Phi$ is an ancillary proper scalar field and $\vec p$ is an
ancillary proper vector field. Then $L$ is invariant under the
change of inertial or non-inertial cartesian coordinate system. We
investigate critical points of the functional
\begin{equation}\label{btfffygtgyggyDD}
J=\int_0^T\int_{\mathbb{R}^3}L\left(\vec A,\Psi,\vec v,\Phi,\vec
p,\vec x,t\right)d\vec x dt.
\end{equation}
We denote
\begin{equation}\label{guigjgjffghDD}
\begin{cases}
\vec D=-\nabla_{\vec x}\Psi-\frac{1}{c}\frac{\partial\vec
A}{\partial t}+\frac{1}{c}\vec
v\times curl_{\vec x}\vec A\\
\vec B=curl_{\vec x}\vec A
\\
\vec E=-\nabla_{\vec x}\Psi-\frac{1}{c}\frac{\partial\vec A}{\partial t}=\vec D-\frac{1}{c}\vec v\times\vec B\\
\vec H=curl_{\vec x}\vec A+\frac{1}{c}\vec
v\times\left(-\nabla_{\vec x}\Psi-\frac{1}{c}\frac{\partial\vec
A}{\partial t}+\frac{1}{c}\vec v\times curl_{\vec x}\vec
A\right)=\vec B+\frac{1}{c}\vec v\times\vec D.
\end{cases}
\end{equation}
Then by \er{guigjgjffghDD} we have:
\begin{equation}\label{guigjgjffghjhkkgDD}
\begin{cases}
curl_{\vec x}\vec E+\frac{1}{c}\frac{\partial\vec B}{\partial t}=0\\
div_{\vec x}\vec B=0.
\end{cases}
\end{equation}
Moreover by \er{vhfffngghkjgghDD},  \er{apfrm3} and \er{apfrm6} we
have
\begin{equation}\label{vhfffngghkjgghggtghjgfhjyuiutDD}
\frac{\delta L}{\delta \vec p}=-div_{\vec x}\left(d_{\vec x}\vec
v+\left\{d_{\vec x}\vec v\right\}^T\right)+2\nabla_{\vec
x}\left(div_{\vec x}\vec v\right)=curl_{\vec x}\left(curl_{\vec
x}\vec v\right)=0,
\end{equation}
\begin{equation}\label{vhfffngghkjgghggtghjgfhjDD}
\frac{\delta L}{\delta \Phi}=-\frac{1}{4\pi
G}\left(\frac{\partial}{\partial t}\left\{div_{\vec x}\vec
v\right\}+\vec v\cdot\nabla_{\vec x}\left(div_{\vec x}\vec
v\right)+\frac{1}{4}\left|d_{\vec x}\vec v+\left\{d_{\vec x}\vec
v\right\}^T\right|^2\right)-\frac{1}{4\pi G}\Delta_{\vec x}\Phi=0,
\end{equation}
\begin{multline}\label{vhfffngghkjgghggtghjgfhjoyuiyuDD}
\frac{\delta L}{\delta \vec v}=-\left(\mu\vec u-\mu\vec
v+\frac{1}{4\pi c}\vec D\times\vec B\right)-div_{\vec
x}\left(d_{\vec x}\vec p+\left\{d_{\vec x}\vec
p\right\}^T\right)+2\nabla_{\vec x}\left(div_{\vec x}\vec
p\right)\\+\frac{1}{4\pi G}div_{\vec x}\left\{\left(d_{\vec x}\vec
v+\left\{d_{\vec x}\vec v\right\}^T\right)\Phi\right\}-\frac{1}{2\pi
G}\nabla_{\vec x}\left(\Phi\left(div_{\vec x}\vec
v\right)\right)-\frac{1}{4\pi G}\nabla_{\vec x}\left(\frac{\partial
\Phi}{\partial t}+\vec v\cdot\nabla_{\vec x}
\Phi\right)\\+\frac{1}{4\pi G}\left(div_{\vec x}\vec
v\right)\nabla_{\vec x}\Phi=-\left(\mu\vec u-\mu\vec v+\frac{1}{4\pi
c}\vec D\times\vec B\right)+curl_{\vec x}\left(curl_{\vec x}\vec
p\right)-\frac{1}{4\pi G}\Phi\,curl_{\vec x}\left(curl_{\vec x}\vec
v\right)\\-\frac{1}{4\pi G}\left(\frac{\partial}{\partial
t}\left(\nabla_{\vec x}\Phi\right)-curl_{\vec x}\left(\vec
v\times\nabla_{\vec x}\Phi\right)+\left(\Delta_{\vec
x}\Phi\right)\vec v\right)=0,
\end{multline}
\begin{equation}\label{vhfffngghkjgghggtghjgfhjhjkghghDD}
\frac{\delta L}{\delta \Psi}=\frac{1}{4\pi}div_{\vec x}\vec
D-\rho=0,
\end{equation}
and
\begin{equation}\label{vhfffngghkjgghggtghjgfhjhjkghghyuiuuDD}
\frac{\delta L}{\delta \vec A}=\frac{1}{c}\vec j+\frac{1}{4\pi
c}\frac{\partial\vec D}{\partial t}-\frac{1}{4\pi}curl_{\vec x}\vec
B-\frac{1}{4\pi c}curl_{\vec x}\left(\vec v\times \vec
D\right)=\frac{1}{c}\vec j+\frac{1}{4\pi c}\frac{\partial\vec
D}{\partial t}-\frac{1}{4\pi}curl_{\vec x}\vec H=0.
\end{equation}
So
\begin{equation}\label{guigjgjffghguygjyfDD}
\begin{cases}
curl_{\vec x}\vec H=\frac{4\pi}{c}\vec
j+\frac{1}{c}\frac{\partial\vec D}{\partial
t}\\
div_{\vec x}\vec D=4\pi\rho\\
curl_{\vec x}\vec E+\frac{1}{c}\frac{\partial\vec B}{\partial t}=0\\
div_{\vec x}\vec B=0\\
\vec E=\vec D-\frac{1}{c}\vec v\times\vec B\\
\vec H=\vec B+\frac{1}{c}\vec v\times\vec D\\
curl_{\vec x}\left(curl_{\vec x}\vec v\right)=0
\\
\frac{\partial}{\partial t}\left\{div_{\vec x}\vec v\right\}+\vec
v\cdot\nabla_{\vec x}\left(div_{\vec x}\vec
v\right)+\frac{1}{4}\left|d_{\vec x}\vec v+\left\{d_{\vec x}\vec
v\right\}^T\right|^2=-\Delta_{\vec x}\Phi\\
\left(\mu\vec u-\mu\vec v+\frac{1}{4\pi c}\vec D\times\vec
B\right)=curl_{\vec x}\left(curl_{\vec x}\vec p\right)-\frac{1}{4\pi
G}\left(\frac{\partial}{\partial t}\left(\nabla_{\vec
x}\Phi\right)-curl_{\vec x}\left(\vec v\times\nabla_{\vec
x}\Phi\right)+\left(\Delta_{\vec x}\Phi\right)\vec v\right).
\end{cases}
\end{equation}
In particular, using continuum equation $\partial_t\mu+div_{\vec
x}\left(\mu\vec u\right)=0$ from the last equality in
\er{guigjgjffghguygjyfDD} we deduce
\begin{equation*}
\frac{\partial}{\partial t}\left(\frac{1}{4\pi G}\Delta_{\vec
x}\Phi-\mu\right)+div_{\vec x}\left\{\left(\frac{1}{4\pi
G}\Delta_{\vec x}\Phi-\mu\right)\vec v\right\}=-div_{\vec
x}\left\{\frac{1}{4\pi c}\vec D\times\vec B\right\}.
\end{equation*}
Thus denoting $Q=\Delta_{\vec x}\Phi/4\pi G-\mu$ we deduce
\begin{equation}\label{guigjgjffghguygjyfghggDD}
\begin{cases}
curl_{\vec x}\vec H=\frac{4\pi}{c}\vec
j+\frac{1}{c}\frac{\partial\vec D}{\partial
t}\\
div_{\vec x}\vec D=4\pi\rho\\
curl_{\vec x}\vec E+\frac{1}{c}\frac{\partial\vec B}{\partial t}=0\\
div_{\vec x}\vec B=0\\
\vec E=\vec D-\frac{1}{c}\vec v\times\vec B\\
\vec H=\vec B+\frac{1}{c}\vec v\times\vec D\\
curl_{\vec x}\left(curl_{\vec x}\vec v\right)=0
\\
\frac{\partial}{\partial t}\left\{div_{\vec x}\vec v\right\}+\vec
v\cdot\nabla_{\vec x}\left(div_{\vec x}\vec
v\right)+\frac{1}{4}\left|d_{\vec x}\vec v+\left\{d_{\vec x}\vec
v\right\}^T\right|^2=-4\pi G(\mu+Q)\\
\frac{\partial Q}{\partial t}+div_{\vec x}\left(Q\vec
v\right)=-div_{\vec x}\left\{\frac{1}{4\pi c}\vec D\times\vec
B\right\}.
\end{cases}
\end{equation}

\subsection{Lagrangian of the unified
Gravitational-Electromagnetic field in the case of some possible
alternative model of the gravity}\label{ghjfhgfdhdnw22mm} Consider
$\vec k$ to be the vectorial potential of the inertia, which is a
generally trivial speed-like vector field, assumed to be fixed in
every fixed inertial or non-inertial cartesian coordinate system
(see Definition \ref{jjhgyftdtd}). Given known the distribution of
inertial mass density of some continuum medium $\mu:=\mu(\vec x,t)$,
the field of velocities of this medium $\vec u:=\vec u(\vec x,t)$,
the charge density $\rho:=\rho(\vec x,t)$ and the current density
$\vec j:=\vec j(\vec x,t)$ consider a Lagrangian density $L$ defined
by
\begin{multline}\label{vhfffngghkjgghDDnw22mm}
L\left(\vec A,\Psi,\vec v,\Phi_0,\vec h,\vec x,t\right):=\\
\frac{1}{8\pi}\left|-\nabla_{\vec
x}\Psi-\frac{1}{c}\frac{\partial\vec A}{\partial t}+\frac{1}{c}\vec
v\times curl_{\vec x}\vec A\right|^2-\frac{1}{8\pi}\left|curl_{\vec
x}\vec A\right|^2-\left(\rho\Psi-\frac{1}{c}\vec A\cdot\vec
j\right)+\frac{\mu}{2}\left|\vec u-\vec v\right|^2
\\
-\frac{c^2}{8\pi G}\left|-\nabla_{\vec x}\Phi_0
-\frac{1}{c}\frac{\partial\vec h}{\partial t}
+\frac{1}{c}\vec v\times curl_{\vec x}\vec h-\frac{1}{c}\nabla_{\vec
x}\left(\vec h\cdot\vec v\right)\right|^2+\frac{c^2}{8\pi
G}\left|curl_{\vec x}\vec h\right|^2\\+\frac{c^2\beta}{4\pi
G}\left(\frac{1}{4}\left|d_{\vec x}\vec v+\left\{d_{\vec x}\vec
v\right\}^T\right|^2-\left|\Div_{\vec x}\vec v\right|^2\right) ,
\end{multline}
where
\begin{equation}\label{btfffygtgyggyDDnw22jkhk22mm}
\vec h=\vec v-\vec
k\quad\text{and}\quad\Phi_0=-\frac{1}{2c}\left|\vec v-\vec
k\right|^2,
\end{equation}
and $\beta\in \mathbb{R}$ is some constant. In other words we have:
\begin{multline}\label{vhfffngghkjgghDDnw22jkjkjkjk22gughhg22mm}
L\left(\vec A,\Psi,\vec v,\Phi_0,\vec h,\vec
x,t\right)=L_1\left(\vec A,\Psi,\vec v,\vec
x,t\right):=\\
\frac{1}{8\pi}\left|-\nabla_{\vec
x}\Psi-\frac{1}{c}\frac{\partial\vec A}{\partial t}+\frac{1}{c}\vec
v\times curl_{\vec x}\vec A\right|^2-\frac{1}{8\pi}\left|curl_{\vec
x}\vec A\right|^2-\left(\rho\Psi-\frac{1}{c}\vec A\cdot\vec
j\right)+\frac{\mu}{2}\left|\vec u-\vec v\right|^2
\\
-\frac{c^2}{8\pi G}\left|
\frac{1}{c}\nabla_{\vec x}\left(\frac{1}{2}\left|\vec v-\vec
k\right|^2\right)-\frac{1}{c}\frac{\partial}{\partial t}\left(\vec
v-\vec k\right)
+\frac{1}{c}\vec v\times curl_{\vec x}\left(\vec v-\vec
k\right)-\frac{1}{c}\nabla_{\vec x}\left(\left(\vec v-\vec
k\right)\cdot\vec v\right)\right|^2\\+\frac{c^2}{8\pi
G}\left|curl_{\vec x}\left(\vec v-\vec
k\right)\right|^2+\frac{c^2\beta}{4\pi
G}\left(\frac{1}{4}\left|d_{\vec x}\vec v+\left\{d_{\vec x}\vec
v\right\}^T\right|^2-\left|\Div_{\vec x}\vec v\right|^2\right).
\end{multline}
In particular, in the inertial coordinate system where $d_{\vec
x}\vec k=0$ and $\partial_t\vec k=0$ we have:
\begin{multline}\label{vhfffngghkjgghDDnw22jkjkjkjk22gughhg22hkhjhjhjh22kljlkjk22mma}
L\left(\vec A,\Psi,\vec v,\Phi_0,\vec h,\vec
x,t\right)=L_1\left(\vec A,\Psi,\vec v,\vec
x,t\right):=\\
\frac{1}{8\pi}\left|-\nabla_{\vec
x}\Psi-\frac{1}{c}\frac{\partial\vec A}{\partial t}+\frac{1}{c}\vec
v\times curl_{\vec x}\vec A\right|^2-\frac{1}{8\pi}\left|curl_{\vec
x}\vec A\right|^2-\left(\rho\Psi-\frac{1}{c}\vec A\cdot\vec
j\right)+\frac{\mu}{2}\left|\vec u-\vec v\right|^2
\\
-\frac{c^2}{8\pi G}\left|
-\frac{1}{c}\frac{\partial\vec v}{\partial t}
+\frac{1}{c}\vec v\times curl_{\vec x}\vec v-\frac{1}{c}\nabla_{\vec
x}\left(\frac{1}{2}\left|\vec
v\right|^2\right)\right|^2+\frac{c^2}{8\pi G}\left|curl_{\vec x}\vec
v\right|^2\\+\frac{c^2\beta}{4\pi G}\left(\frac{1}{4}\left|d_{\vec
x}\vec v+\left\{d_{\vec x}\vec v\right\}^T\right|^2-\left|\Div_{\vec
x}\vec v\right|^2\right),
\end{multline}
Note here the advantage of inertial coordinate systems, where $L$
and $L_1$ are complectly independent of the vectorial potential of
the inertia $\vec k$. Furthermore, by \er{apfrm9} we rewrite
\er{vhfffngghkjgghDDnw22jkjkjkjk22gughhg22hkhjhjhjh22kljlkjk22mma}
as:
\begin{multline}\label{vhfffngghkjgghDDnw22jkjkjkjk22gughhg22hkhjhjhjh22kljlkjk22mm}
L\left(\vec A,\Psi,\vec v,\Phi_0,\vec h,\vec
x,t\right)=L_1\left(\vec A,\Psi,\vec v,\vec
x,t\right):=\\
\frac{1}{8\pi}\left|-\nabla_{\vec
x}\Psi-\frac{1}{c}\frac{\partial\vec A}{\partial t}+\frac{1}{c}\vec
v\times curl_{\vec x}\vec A\right|^2-\frac{1}{8\pi}\left|curl_{\vec
x}\vec A\right|^2-\left(\rho\Psi-\frac{1}{c}\vec A\cdot\vec
j\right)+\frac{\mu}{2}\left|\vec u-\vec v\right|^2
\\
-\frac{c^2}{8\pi G}\left|
-\frac{1}{c}\left(\frac{\partial\vec v}{\partial t}+d_{\vec x}\vec
v\cdot\vec v\right)\right|^2+\frac{c^2}{8\pi G}\left|curl_{\vec
x}\vec v\right|^2+\frac{c^2\beta}{4\pi
G}\left(\frac{1}{4}\left|d_{\vec x}\vec v+\left\{d_{\vec x}\vec
v\right\}^T\right|^2-\left|\Div_{\vec x}\vec v\right|^2\right).
\end{multline}
Then, using Proposition \ref{yghgjtgyrtrt} by
\er{vhfffngghkjgghDDnw22mm}, \er{btfffygtgyggyDDnw22jkhk22mm} and
\er{vhfffngghkjgghDDnw22jkjkjkjk22gughhg22mm} we deduce that $L$ and
$L_1$ are invariant under the change of inertial or non-inertial
cartesian coordinate system, provided that $\vec h$ and $\vec A$ are
proper vector fields, $\vec v$ is a speed-like vector field and
$\Phi_0$ and $\Psi_0:=\Psi-\frac{1}{c}\vec A\cdot\vec v$ are proper
scalar fields.

We investigate critical points of the functional
\begin{equation}\label{btfffygtgyggyDDnw22mm}
J=\int_0^T\int_{\mathbb{R}^3}L_1\left(\vec A,\Psi,\vec v,\vec
x,t\right)d\vec x dt.
\end{equation}
We denote
\begin{equation}\label{guigjgjffghDDnw22mm}
\begin{cases}
\Psi_0=\Psi-\frac{1}{c}\vec A\cdot\vec v\\
\vec D=-\nabla_{\vec x}\Psi-\frac{1}{c}\frac{\partial\vec
A}{\partial t}+\frac{1}{c}\vec v\times curl_{\vec x}\vec
A=-\nabla_{\vec x}\Psi_0-\frac{1}{c}\frac{\partial\vec A}{\partial
t}+\frac{1}{c}\vec
v\times curl_{\vec x}\vec A-\frac{1}{c}\nabla_{\vec x}\left(\vec A\cdot\vec v\right)\\
\vec B=curl_{\vec x}\vec A
\\
\vec E=-\nabla_{\vec x}\Psi-\frac{1}{c}\frac{\partial\vec A}{\partial t}=\vec D-\frac{1}{c}\vec v\times\vec B\\
\vec H=curl_{\vec x}\vec A+\frac{1}{c}\vec
v\times\left(-\nabla_{\vec x}\Psi-\frac{1}{c}\frac{\partial\vec
A}{\partial t}+\frac{1}{c}\vec v\times curl_{\vec x}\vec
A\right)=\vec B+\frac{1}{c}\vec v\times\vec D.
\end{cases}
\end{equation}
and
\begin{equation}\label{guigjgjffghDDnwjkhhkhkhjklh22mm}
\begin{cases}
\vec R=-\nabla_{\vec x}\Phi_0-\frac{1}{c}\frac{\partial\vec
h}{\partial t}
+\frac{1}{c}\vec v\times curl_{\vec x}\vec
h-\frac{1}{c}\nabla_{\vec x}\left(\vec h\cdot\vec v\right)\\
\vec Q=curl_{\vec x}\vec h,
\end{cases}
\end{equation}
where $\vec h$ is a proper vector field and $\Phi_0$ is a proper
scalar field that are given by \er{btfffygtgyggyDDnw22jkhk22mm}. In
other words,
\begin{equation}\label{guigjgjffghDDnwjkhhkhkhjklh22ljjljkhk22mm}
\begin{cases}
\vec R=\frac{1}{c}\nabla_{\vec x}\left(\frac{1}{2}\left|\vec v-\vec
k\right|^2\right)-\frac{1}{c}\frac{\partial}{\partial t}\left(\vec
v-\vec k\right)
+\frac{1}{c}\vec v\times curl_{\vec x}\left(\vec v-\vec
k\right)-\frac{1}{c}\nabla_{\vec x}\left(\left(\vec v-\vec
k\right)\cdot\vec v\right)\\
\vec Q=curl_{\vec x}(\vec v-\vec k),
\end{cases}
\end{equation}
and in inertial coordinate system where $d_{\vec x}\vec k=0$ and
$\partial_t\vec k=0$ we also have:
\begin{equation}\label{guigjgjffghDDnwjkhhkhkhjklh22ljjljkhk22hhmm}
\begin{cases}
\vec R=-\frac{1}{c}\left(\frac{\partial\vec v}{\partial t}+d_{\vec
x}\vec
v\cdot\vec v\right)\\
\vec Q=curl_{\vec x}\vec v.
\end{cases}
\end{equation}
As before, by \er{guigjgjffghDDnw22mm} and
\er{guigjgjffghDDnwjkhhkhkhjklh22mm} and Proposition
\ref{yghgjtgyrtrt} we infer that both $\vec D$, $\vec B$ and $\vec
R$, $\vec Q$ are proper vector fields. Next, by
\er{guigjgjffghDDnw22mm} we have:
\begin{equation}\label{guigjgjffghjhkkgDDnw22mmff}
\begin{cases}
curl_{\vec x}\vec D+\frac{1}{c}\frac{\partial\vec B}{\partial t}
-\frac{1}{c}\,curl_{\vec x}\left(\vec v\times\vec B\right)=curl_{\vec x}\vec E+\frac{1}{c}\frac{\partial\vec B}{\partial t}=0\\
div_{\vec x}\vec B=0,
\end{cases}
\end{equation}
and by \er{guigjgjffghDDnwjkhhkhkhjklh22mm} we have:
\begin{equation}\label{guigjgjffghjhkkgDDnw22mm}
\begin{cases}
curl_{\vec x}\vec R+\frac{1}{c}\frac{\partial\vec Q}{\partial t}-\frac{1}{c}\,curl_{\vec x}\left(\vec v\times\vec Q\right)=0\\
div_{\vec x}\vec Q=0.
\end{cases}
\end{equation}
Furthermore, by \er{guigjgjffghDDnwjkhhkhkhjklh22ljjljkhk22mm} and
\er{apfrm9} we deduce
\begin{multline}\label{guigjgjffghDDnwjkhhkhkhjklh22ljjljkhk22hkjhjhjh22mm}
\vec R=\frac{1}{c}\nabla_{\vec x}\left(\frac{1}{2}\left|\vec v-\vec
k\right|^2\right)-\frac{1}{c}\left\{d_{\vec x}\vec
v\right\}^T\cdot\left(\vec v-\vec
k\right)-\frac{1}{c}\frac{\partial}{\partial t}\left(\vec v-\vec
k\right)-\frac{1}{c}\left\{d_{\vec x}\left(\vec v-\vec
k\right)\right\}^T\cdot\vec v
+\frac{1}{c}\vec v\times curl_{\vec x}\left(\vec v-\vec k\right)\\=
-\frac{1}{c}\left\{d_{\vec x}\vec k\right\}^T\cdot\left(\vec v-\vec
k\right)-\frac{1}{c}\frac{\partial}{\partial t}\left(\vec v-\vec
k\right)-\frac{1}{c}d_{\vec x}\left(\vec v-\vec k\right)\cdot\vec v
\\
\quad\text{and}\quad \vec Q=curl_{\vec x}(\vec v-\vec k),
\end{multline}
and thus, since $d_{\vec x}\vec k+\{d_{\vec x}\vec k\}^T=0$,
$curl_{\vec x}(curl_{\vec x}\vec k)=0$ and $\Div_{\vec x}\vec k=0$
we infer
\begin{multline}\label{guigjgjffghDDnwjkhhkhkhjklh22ljjljkhk22hkjhjhjh22jljjkhjgh22hhjhmm}
\Div_{\vec x}\vec R=
-\frac{1}{c}\left(\vec v-\vec k\right)\cdot\Delta_{\vec x}\vec
k-\frac{1}{c}tr\left(\{d_{\vec x}\vec k\}^T\cdot d_{\vec
x}\left(\vec v-\vec
k\right)\right)-\frac{1}{c}\frac{\partial}{\partial
t}\left(\Div_{\vec x}\vec v\right)-\frac{1}{c}d_{\vec
x}\left(\Div_{\vec x}\vec v\right)\cdot\vec
v\\-\frac{1}{c}tr\left(d_{\vec x}\left(\vec v-\vec
k\right)\cdot
d_{\vec x}\vec v
\right)\\
=-\frac{1}{c}\frac{\partial}{\partial t}\left(\Div_{\vec x}\vec
v\right)-\frac{1}{c}d_{\vec x}\left(\Div_{\vec x}\vec
v\right)\cdot\vec v-\frac{1}{c}tr\left(d_{\vec x}\left(\vec v-\vec
k\right)\cdot
d_{\vec x}(\vec v-\vec k)
\right)\\
=-\frac{1}{c}\frac{\partial}{\partial t}\left(\Div_{\vec x}\vec
v\right)-\frac{1}{c}d_{\vec x}\left(\Div_{\vec x}\vec
v\right)\cdot\vec v-\frac{1}{4c}\left|d_{\vec x}(\vec v-\vec
k)+\left\{d_{\vec x}(\vec v-\vec
k)\right\}^T\right|^2+\frac{1}{4c}\left|d_{\vec x}(\vec v-\vec
k)-\left\{d_{\vec x}(\vec v-\vec k)\right\}^T\right|^2
\\
=-\frac{1}{c}\frac{\partial}{\partial t}\left(\Div_{\vec x}\vec
v\right)-\frac{1}{c}d_{\vec x}\left(\Div_{\vec x}\vec
v\right)\cdot\vec v-\frac{1}{4c}\left|d_{\vec x}\vec
v+\left\{d_{\vec x}\vec v\right\}^T\right|^2+\frac{1}{2c}\left|\vec
Q\right|^2
\\
\quad\text{and}\quad curl_{\vec x}\vec Q=curl_{\vec x}(curl_{\vec
x}\vec v).
\end{multline}
So,
\begin{multline}\label{guigjgjffghDDnwjkhhkhkhjklh22ljjljkhk22hkjhjhjh22jljjkhjgh22hhjhjijhjhjtmm}
\frac{1}{c}\left(\frac{\partial}{\partial t}\left(\Div_{\vec x}\vec
v\right)+d_{\vec x}\left(\Div_{\vec x}\vec v\right)\cdot\vec
v+\frac{1}{4}\left|d_{\vec x}\vec v+\left\{d_{\vec x}\vec
v\right\}^T\right|^2\right)=\frac{1}{2c}\left|\vec
Q\right|^2-\Div_{\vec x}\vec R
\\
\quad\text{and}\quad curl_{\vec x}(curl_{\vec x}\vec v)=curl_{\vec
x}\vec Q,
\end{multline}
In other words,
\begin{multline}\label{guigjgjffghDDnwjkhhkhkhjklh22ljjljkhk22hkjhjhjh22jljjkhjgh22hhjhjijhjhjtjljjmm}
\frac{1}{c}\left(\frac{\partial}{\partial t}\left(\Div_{\vec x}\vec
v\right)+\Div_{\vec x}\left\{\left(\Div_{\vec x}\vec v\right)\vec
v\right\}+\frac{1}{4}\left|d_{\vec x}\vec v+\left\{d_{\vec x}\vec
v\right\}^T\right|^2-\left|\Div_{\vec x}\vec
v\right|^2\right)=\frac{1}{2c}\left|\vec Q\right|^2-\Div_{\vec
x}\vec R
\\
\quad\text{and}\quad curl_{\vec x}(curl_{\vec x}\vec v)=curl_{\vec
x}\vec Q.
\end{multline}
Moreover, by \er{vhfffngghkjgghDDnw22mm}, and  \er{apfrm3}
we have
\begin{equation}\label{vhfffngghkjgghggtghjgfhjyuiutDDnw22mm}
\frac{\delta L}{\delta \vec h}=
\frac{c^2}{4\pi G}curl_{\vec x}\vec Q
-\frac{c}{4\pi G}\left(\frac{\partial\vec R}{\partial t}-curl_{\vec
x}\left\{\vec v\times\vec R\right\}+\left(div_{\vec x}\vec
R\right)\vec v\right),
\end{equation}
\begin{equation}\label{vhfffngghkjgghggtghjgfhjoyuiyuDDnw22mmhkkhjhmm}
\frac{\delta L}{\delta{\Phi_0}}=-\frac{c^2}{4\pi G}\left(div_{\vec
x}\vec R\right).
\end{equation}
and
\begin{equation}\label{vhfffngghkjgghggtghjgfhjoyuiyuDDnw22mm}
\frac{\delta L}{\delta \vec v}=-\left(\mu\vec u-\mu\vec
v+\frac{1}{4\pi c}\vec D\times\vec B-\frac{c}{4\pi G}\vec
R\times\vec Q\right)+\frac{c^2\beta}{4\pi G}\,curl_{\vec
x}(curl_{\vec x}\vec v)-\frac{c}{4\pi G}\left(div_{\vec x}\vec
R\right)\vec h.
\end{equation}
Thus, since by \er{btfffygtgyggyDDnw22jkhk22mm} we have
\begin{equation}\label{jkhlgytguoiyhioyhmm}
L_1\left(\vec A,\Psi,\vec v,\vec x,t\right)=L\left(\vec A,\Psi,\vec
v,-\frac{1}{2c}\left|\vec v-\vec k\right|^2,(\vec v-\vec k),\vec
x,t\right)=L\left(\vec A,\Psi,\vec v,\Phi_0,\vec h,\vec x,t\right),
\end{equation}
by Chain rule we have
\begin{equation}\label{jkhlgytguoiyhioyh11mm}
\frac{\delta L_1}{\delta \vec v}\left(\vec A,\Psi,\vec v,\vec
x,t\right)=
\frac{\delta L}{\delta \vec v}+\frac{\delta L}{\delta \vec
h}-\frac{1}{c}\frac{\delta L}{\delta {\Phi_0}}\left(\vec v-\vec
k\right)
\end{equation}
Therefore, using \er{vhfffngghkjgghggtghjgfhjyuiutDDnw22mm},
\er{vhfffngghkjgghggtghjgfhjoyuiyuDDnw22mmhkkhjhmm}
\er{vhfffngghkjgghggtghjgfhjoyuiyuDDnw22mm} in
\er{jkhlgytguoiyhioyh11mm} we deduce:
\begin{multline}\label{vhfffngghkjgghggtghjgfhjyuiutDDnw22kllkjkllkjl22klljjkjk22mm}
\left(\mu\vec u-\mu\vec v+\frac{1}{4\pi c}\vec D\times\vec
B-\frac{c}{4\pi G}\vec R\times\vec Q\right)=\\ \frac{c^2\beta}{4\pi
G}\,curl_{\vec x}(curl_{\vec x}\vec v)+\frac{c^2}{4\pi G}curl_{\vec
x}\vec Q
-\frac{c}{4\pi G}\left(\frac{\partial\vec R}{\partial t}-curl_{\vec
x}\left\{\vec v\times\vec R\right\}+\left(div_{\vec x}\vec
R\right)\vec v\right).
\end{multline}
Moreover,
\begin{equation}\label{vhfffngghkjgghggtghjgfhjhjkghghDDnw22mm}
\frac{\delta L}{\delta \Psi}=\frac{1}{4\pi}div_{\vec x}\vec
D-\rho=0,
\end{equation}
and
\begin{equation}\label{vhfffngghkjgghggtghjgfhjhjkghghyuiuuDDnw22mm}
\frac{\delta L}{\delta \vec A}=\frac{1}{c}\vec j+\frac{1}{4\pi
c}\frac{\partial\vec D}{\partial t}-\frac{1}{4\pi}curl_{\vec x}\vec
B-\frac{1}{4\pi c}curl_{\vec x}\left(\vec v\times \vec
D\right)=\frac{1}{c}\vec j+\frac{1}{4\pi c}\frac{\partial\vec
D}{\partial t}-\frac{1}{4\pi}curl_{\vec x}\vec H=0.
\end{equation}
So
\begin{equation}\label{guigjgjffghguygjyfDDnw22mmd}
\begin{cases}
curl_{\vec x}\vec B=\frac{4\pi}{c}\vec
j+\frac{1}{c}\frac{\partial\vec D}{\partial t}-\frac{1}{c}curl_{\vec
x}\left(\vec v\times \vec D\right)
\\
div_{\vec x}\vec D=4\pi\rho
\\
curl_{\vec x}\vec D+\frac{1}{c}\frac{\partial\vec B}{\partial t}
-\frac{1}{c}\,curl_{\vec x}\left(\vec v\times\vec B\right)=0\\
div_{\vec x}\vec B=0\\
curl_{\vec x}\vec R+\frac{1}{c}\frac{\partial\vec Q}{\partial t}-\frac{1}{c}curl_{\vec x}\left(\vec v\times\vec Q\right)=0\\
div_{\vec x}\vec Q=0
\\
\frac{1}{c}\left(\frac{\partial}{\partial t}\left(\Div_{\vec x}\vec
v\right)+\vec v\cdot\nabla_{\vec x}\left(\Div_{\vec x}\vec
v\right)+\frac{1}{4}\left|d_{\vec x}\vec v+\left\{d_{\vec x}\vec
v\right\}^T\right|^2\right)=\frac{1}{2c}\left|\vec
Q\right|^2-\Div_{\vec x}\vec R
\\
curl_{\vec x}(curl_{\vec x}\vec v)=curl_{\vec x}\vec Q
\\
\frac{4\pi G}{c^2}\left(\mu\vec u-\mu\vec v+\frac{1}{4\pi c}\vec
D\times\vec
B-\frac{c}{4\pi G}\vec R\times\vec Q\right)=\\
(1+\beta)\,curl_{\vec x}\vec Q
-\frac{1}{c}\left(\frac{\partial\vec R}{\partial t}-curl_{\vec
x}\left\{\vec v\times\vec R\right\}+\left(div_{\vec x}\vec
R\right)\vec v\right),
\end{cases}
\end{equation}
and by \er{guigjgjffghDDnwjkhhkhkhjklh22ljjljkhk22hhmm} in the
inertial frame we have:
\begin{equation}\label{guigjgjffghguygjyfDDnw22khhjhgg22mmd}
\begin{cases}
curl_{\vec x}\vec B=\frac{4\pi}{c}\vec
j+\frac{1}{c}\frac{\partial\vec D}{\partial t}-\frac{1}{c}curl_{\vec
x}\left(\vec v\times \vec
D\right)\\
div_{\vec x}\vec D=4\pi\rho\\
curl_{\vec x}\vec D+\frac{1}{c}\frac{\partial\vec B}{\partial t}
-\frac{1}{c}\,curl_{\vec x}\left(\vec v\times\vec B\right)=0\\
div_{\vec x}\vec B=0\\
\vec R=-\frac{1}{c}\left(\frac{\partial\vec v}{\partial t}+d_{\vec
x}\vec
v\cdot\vec v\right)\\
\vec Q=curl_{\vec x}\vec v
\\
\frac{4\pi G}{c^2}\left(\mu\vec u-\mu\vec v+\frac{1}{4\pi c}\vec
D\times\vec
B-\frac{c}{4\pi G}\vec R\times\vec Q\right)=\\
(1+\beta)\,curl_{\vec x}\vec Q
-\frac{1}{c}\left(\frac{\partial\vec R}{\partial t}-curl_{\vec
x}\left\{\vec v\times\vec R\right\}+\left(div_{\vec x}\vec
R\right)\vec v\right).
\end{cases}
\end{equation}
Furthermore,  taking $\Div_{\vec x}$ of the both sides of the last
equality in \er{guigjgjffghguygjyfDDnw22mmd} and using continuum
equation $\partial_t\mu+\Div_{\vec x}\left(\mu\vec u\right)=0$ we
deduce
\begin{multline}\label{guigjgjffghguygjyfDDnw22mmz}
-\left(\partial_t\mu+\Div_{\vec x}\left(\mu\vec
v\right)\right)+\Div_{\vec x}\left\{\frac{1}{4\pi c}\vec D\times\vec
B-\frac{c}{4\pi  G}\vec R\times\vec Q\right\}=\\ \Div_{\vec
x}\left(\mu\vec u-\mu\vec v+\frac{1}{4\pi c}\vec D\times\vec
B-\frac{c}{4\pi  G}\vec R\times\vec Q\right)=\\
-\frac{c}{4\pi  G}\left(\frac{\partial}{\partial t}\left(\Div_{\vec
x}\vec R\right)+\Div_{\vec x}\left\{\left(div_{\vec x}\vec
R\right)\vec v\right\}\right),
\end{multline}
Therefore, considering the proper scalar quantity $Q_0$, that we
call the field mass, which satisfies
\begin{equation}\label{jtgiktiutiurtfirfmmhfffmm}
Q_0:=-\mu+\frac{c}{4\pi  G}\,\Div_{\vec x}\vec R,
\end{equation}
by \er{guigjgjffghguygjyfDDnw22mmz} we deduce
\begin{equation}\label{jtgiktiutiurtfirfmm}
\frac{\partial Q_0}{\partial t}+div_{\vec x}\left(Q_0\vec
v\right)=-\Div_{\vec x}\left\{\frac{1}{4\pi c}\vec D\times\vec
B-\frac{c}{4\pi  G}\vec R\times\vec Q\right\}.
\end{equation}
Thus, we rewrite \er{guigjgjffghguygjyfDDnw22mmd} as:
\begin{equation}\label{guigjgjffghguygjyfDDnw22mm}
\begin{cases}
curl_{\vec x}\vec B=\frac{4\pi}{c}\left(\vec j-\rho\vec
v\right)+\frac{1}{c}\frac{\partial\vec D}{\partial
t}-\frac{1}{c}curl_{\vec x}\left(\vec v\times \vec
D\right)+\frac{1}{c}\left(div_{\vec x}\vec D\right)\vec v,
\\
div_{\vec x}\vec D=4\pi\rho,
\\
curl_{\vec x}\vec D+\frac{1}{c}\frac{\partial\vec B}{\partial t}
-\frac{1}{c}\,curl_{\vec x}\left(\vec v\times\vec B\right)=0,\\
div_{\vec x}\vec B=0,\\
curl_{\vec x}\vec R+\frac{1}{c}\frac{\partial\vec Q}{\partial t}-\frac{1}{c}curl_{\vec x}\left(\vec v\times\vec Q\right)=0,\\
div_{\vec x}\vec Q=0,
\\
\frac{4\pi G}{c^2}\left(\mu\vec u-\mu\vec v+\frac{1}{4\pi c}\vec
D\times\vec
B-\frac{c}{4\pi G}\vec R\times\vec Q\right)=\\
(1+\beta)\,curl_{\vec x}\vec Q
-\frac{1}{c}\left(\frac{\partial\vec R}{\partial t}-curl_{\vec
x}\left\{\vec v\times\vec R\right\}+\left(div_{\vec x}\vec
R\right)\vec v\right),\\
\Div_{\vec x}\vec R=\frac{4\pi G}{c}(\mu+Q_0),
\\
\frac{1}{c}\left(\frac{\partial}{\partial t}\left(\Div_{\vec x}\vec
v\right)+\vec v\cdot\nabla_{\vec x}\left(\Div_{\vec x}\vec
v\right)+\frac{1}{4}\left|d_{\vec x}\vec v+\left\{d_{\vec x}\vec
v\right\}^T\right|^2\right)=\frac{1}{2c}\left|\vec
Q\right|^2-\Div_{\vec x}\vec R,
\\
curl_{\vec x}(curl_{\vec x}\vec v)=curl_{\vec x}\vec Q,
\\
\frac{\partial Q_0}{\partial t}+div_{\vec x}\left(Q_0\vec
v\right)=-\Div_{\vec x}\left\{\frac{1}{4\pi c}\vec D\times\vec
B-\frac{c}{4\pi  G}\vec R\times\vec Q\right\},
\end{cases}
\end{equation}
and we rewrite \er{guigjgjffghguygjyfDDnw22khhjhgg22mmd} in the
inertial frame as:
\begin{equation}\label{guigjgjffghguygjyfDDnw22khhjhgg22mm}
\begin{cases}
curl_{\vec x}\vec B=\frac{4\pi}{c}\vec
j+\frac{1}{c}\frac{\partial\vec D}{\partial t}-\frac{1}{c}curl_{\vec
x}\left(\vec v\times \vec
D\right),\\
div_{\vec x}\vec D=4\pi\rho,\\
curl_{\vec x}\vec D+\frac{1}{c}\frac{\partial\vec B}{\partial t}
-\frac{1}{c}\,curl_{\vec x}\left(\vec v\times\vec B\right)=0,\\
div_{\vec x}\vec B=0,\\
\vec R=-\frac{1}{c}\left(\frac{\partial\vec v}{\partial t}+d_{\vec
x}\vec
v\cdot\vec v\right),\\
\vec Q=curl_{\vec x}\vec v,
\\
\frac{4\pi G}{c^2}\left(\mu\vec u-\mu\vec v+\frac{1}{4\pi c}\vec
D\times\vec
B-\frac{c}{4\pi G}\vec R\times\vec Q\right)=\\
(1+\beta)\,curl_{\vec x}\vec Q
-\frac{1}{c}\left(\frac{\partial\vec R}{\partial t}-curl_{\vec
x}\left\{\vec v\times\vec R\right\}+\left(div_{\vec x}\vec
R\right)\vec v\right),\\
\Div_{\vec x}\vec R=\frac{4\pi G}{c} (\mu+Q_0),\\
\frac{\partial Q_0}{\partial t}+div_{\vec x}\left(Q_0\vec
v\right)=-\Div_{\vec x}\left\{\frac{1}{4\pi c}\vec D\times\vec
B-\frac{c}{4\pi  G}\vec R\times\vec Q\right\}.
\end{cases}
\end{equation}
As before, by Proposition \ref{yghgjtgyrtrt} we deduce that
\er{guigjgjffghguygjyfDDnw22mm} is invariant under the change of
inertial or non-inertial cartesian coordinate systems. Moreover,
\er{guigjgjffghguygjyfDDnw22khhjhgg22mm} is invariant under the
change of inertial cartesian coordinate systems. Note also that,
both, \er{guigjgjffghguygjyfDDnw22mm} in an arbitrary inertial or
non-inertial cartesian coordinate system and
\er{guigjgjffghguygjyfDDnw22khhjhgg22mm} in an arbitrary inertial
cartesian coordinate system, are complectly independent of the
vectorial potential of the inertia $\vec k$.

Finally, note that in the case of large constant $|\beta|\gg 1$ we
have $\vec Q\to 0$ in \er{guigjgjffghguygjyfDDnw22mm} and thus, the
gravity equations \er{guigjgjffghguygjyfDDnw22mm} reduce to the
equations of the Newtonian-type Gravity in the form of
\er{guigjgjffghguygjyfghggDD}. In that case the gravity field
propagates with the infinite speed. On the other hand, in the case
of vanishing constant $\beta=0$ the form of equations for $\vec R$
and $\vec Q$ in \er{guigjgjffghguygjyfDDnw22mm} is completely the
same as the form of the Maxwell equations for $\vec D$ and $\vec B$
in \er{guigjgjffghguygjyfDDnw22mm}, except of the different meaning
of "charges" and "currents" in these two sets of equations. In that
case the electromagnetic and the gravity fields propagate with the
same speed. However, in the mixed case of constant $\beta\sim 1$ the
electromagnetic and the gravity fields propagate with different
finite speeds.

\begin{remark}\label{gygygygyggjh34}
One can wonder: what should be possible values of the vectorial
gravitational potential $\vec v$ in the proximity of the Earth or
another massive body? In remark \ref{gygygygyggjh} we answered this
question in the case of the Newtonian-type gravity , given by
\er{MaxVacFull1ninshtrgravortghhghgjkgghklhjgkghghjjkjhjkkggjkhjkhjjhhfhjhklkhkhjjklzzzyyyNWBWHWPPN222ll}.
In order to answer this question in the case of more general laws of
the gravity, given by \er{guigjgjffghguygjyfDDnw22mm} with arbitrary
constant $\beta$, consider two cartesian coordinate systems:
\underline{non-rotating} system $\bf(*)$ with the center that
coincides with the center of masses of the Earth and
\underline{inertial} system $\bf(**)$ related to some external
cosmic bodies. Assume that the center of masses of the Earth has
place $\vec R(t')$ and velocity $\vec W(t'):=\frac{d\vec
R}{dt'}(t')$ in the coordinate system $\bf(**)$. Thus the change of
coordinate system $\bf(*)$ to coordinate system $\bf(**)$ is given
by
\begin{equation}\label{hjjgghghghyu34}
\begin{cases}
\vec x'=\vec x+\vec R(t),\\
t'=t,
\end{cases}
\end{equation}
and the vectorial gravitational potential $\vec v$, being a speed
like vector field, transforms as
\begin{equation}\label{hjjgghghghyulll34}
\vec v'=\vec v+\vec W(t).
\end{equation}
Next, since the system $\bf(**)$ is inertial, consistently with
\er{guigjgjffghguygjyfDDnw22khhjhgg22mm} in the system $\bf(**)$ we
have
\begin{equation}\label{guigjgjffghguygjyfDDnw22khhjhgg22mm34hhhhhh34}
\begin{cases}
curl_{\vec x'}\vec B'=\frac{4\pi}{c}\left(\vec j'-\rho'\vec
v'\right)+\frac{1}{c}\frac{\partial\vec D'}{\partial
t'}-\frac{1}{c}curl_{\vec x'}\left(\vec v'\times \vec
D'\right)+\frac{1}{c}\left(div_{\vec x'}\vec D'\right)\vec v',
\\
div_{\vec x}\vec D'=4\pi\rho',
\\
curl_{\vec x'}\vec D'+\frac{1}{c}\frac{\partial\vec B'}{\partial t'}
-\frac{1}{c}\,curl_{\vec x'}\left(\vec v'\times\vec B'\right)=0,\\
div_{\vec x'}\vec B'=0,\\
\vec R'=-\frac{1}{c}\left(\frac{\partial\vec v'}{\partial
t'}+d_{\vec x'}\vec
v'\cdot\vec v'\right),\\
\vec Q'=curl_{\vec x'}\vec v',
\\
curl_{\vec x'}\vec R'+\frac{1}{c}\frac{\partial\vec Q'}{\partial t'}-\frac{1}{c}curl_{\vec x'}\left(\vec v'\times\vec Q'\right)=0,\\
div_{\vec x'}\vec Q'=0,
\\
\frac{4\pi G}{c^2}\left(\mu'\vec u'-\mu'\vec v'+\frac{1}{4\pi c}\vec
D'\times\vec
B'-\frac{c}{4\pi G}\vec R'\times\vec Q'\right)=\\
(1+\beta)\,curl_{\vec x'}\vec Q'
-\frac{1}{c}\left(\frac{\partial\vec R'}{\partial t'}-curl_{\vec
x'}\left\{\vec v'\times\vec R'\right\}+\left(div_{\vec x}\vec
R'\right)\vec v'\right),\\
\Div_{\vec x'}\vec R'=\frac{4\pi G}{c}(\mu'+Q'_0),
\\
\frac{\partial Q'_0}{\partial t'}+div_{\vec x'}\left(Q'_0\vec
v'\right)=-\Div_{\vec x'}\left\{\frac{1}{4\pi c}\vec D'\times\vec
B'-\frac{c}{4\pi  G}\vec R'\times\vec Q'\right\},
\end{cases}
\end{equation}
and by \er{guigjgjffghguygjyfDDnw22mm} in the system $\bf(*)$ we
have
\begin{equation}\label{guigjgjffghguygjyfDDnw22mm3434}
\begin{cases}
curl_{\vec x}\vec B=\frac{4\pi}{c}\left(\vec j-\rho\vec
v\right)+\frac{1}{c}\frac{\partial\vec D}{\partial
t}-\frac{1}{c}curl_{\vec x}\left(\vec v\times \vec
D\right)+\frac{1}{c}\left(div_{\vec x}\vec D\right)\vec v,
\\
div_{\vec x}\vec D=4\pi\rho,
\\
curl_{\vec x}\vec D+\frac{1}{c}\frac{\partial\vec B}{\partial t}
-\frac{1}{c}\,curl_{\vec x}\left(\vec v\times\vec B\right)=0,\\
div_{\vec x}\vec B=0,\\
curl_{\vec x}\vec R+\frac{1}{c}\frac{\partial\vec Q}{\partial t}-\frac{1}{c}curl_{\vec x}\left(\vec v\times\vec Q\right)=0,\\
div_{\vec x}\vec Q=0,
\\
\frac{4\pi G}{c^2}\left(\mu\vec u-\mu\vec v+\frac{1}{4\pi c}\vec
D\times\vec
B-\frac{c}{4\pi G}\vec R\times\vec Q\right)=\\
(1+\beta)\,curl_{\vec x}\vec Q
-\frac{1}{c}\left(\frac{\partial\vec R}{\partial t}-curl_{\vec
x}\left\{\vec v\times\vec R\right\}+\left(div_{\vec x}\vec
R\right)\vec v\right),\\
\Div_{\vec x}\vec R=\frac{4\pi G}{c}(\mu+Q_0),
\\
\frac{1}{c}\left(\frac{\partial}{\partial t}\left(\Div_{\vec x}\vec
v\right)+\vec v\cdot\nabla_{\vec x}\left(\Div_{\vec x}\vec
v\right)+\frac{1}{4}\left|d_{\vec x}\vec v+\left\{d_{\vec x}\vec
v\right\}^T\right|^2\right)=\frac{1}{2c}\left|\vec
Q\right|^2-\Div_{\vec x}\vec R,
\\
curl_{\vec x}(curl_{\vec x}\vec v)=curl_{\vec x}\vec Q,
\\
\frac{\partial Q_0}{\partial t}+div_{\vec x}\left(Q_0\vec
v\right)=-\Div_{\vec x}\left\{\frac{1}{4\pi c}\vec D\times\vec
B-\frac{c}{4\pi  G}\vec R\times\vec Q\right\}.
\end{cases}
\end{equation}
On the other hand, by \er{hjjgghghghyu34} and \er{hjjgghghghyulll34}
and using Proposition \ref{yghgjtgyrtrt}
we deduce
\begin{equation}
\label{MaxVacFull1ninshtrgravortghhghgjkgghklhjgkghghjjkjhjkkggjkhjkhjjhhfhjhklkhkhjjklzzzyyyhjggjhgghhjhNWNWBWHWPPN222kkkgtytghjjh34344}
curl_{\vec x}\vec v=curl_{\vec x'}\vec v',
\end{equation}
and
\begin{equation}
\label{MaxVacFull1ninshtrgravortghhghgjkgghklhjgkghghjjkjhjkkggjkhjkhjjhhfhjhklkhkhjjklzzzyyyhjggjhgghhjhNWNWBWHWPPN222kkkgtytghjjhghhg34344}
\frac{d\vec W}{\partial t}(t)+\frac{\partial\vec v}{\partial
t}+d_{\vec x}\vec v\cdot\vec v=\frac{d\vec W}{\partial
t'}(t')+\frac{\partial\vec v}{\partial t'}+d_{\vec x'}\vec
v\cdot\vec W(t')+d_{\vec x'}\vec v\cdot\vec v= \frac{\partial\vec
v'}{\partial t'}+d_{\vec x'}\vec v'\cdot\vec v'.
\end{equation}
Thus since $\vec R$ is a proper field, by
\er{guigjgjffghguygjyfDDnw22khhjhgg22mm34hhhhhh34},
\er{MaxVacFull1ninshtrgravortghhghgjkgghklhjgkghghjjkjhjkkggjkhjkhjjhhfhjhklkhkhjjklzzzyyyhjggjhgghhjhNWNWBWHWPPN222kkkgtytghjjh34344}
and
\er{MaxVacFull1ninshtrgravortghhghgjkgghklhjgkghghjjkjhjkkggjkhjkhjjhhfhjhklkhkhjjklzzzyyyhjggjhgghhjhNWNWBWHWPPN222kkkgtytghjjhghhg34344}
we deduce
\begin{equation}
\label{MaxVacFull1ninshtrgravortghhghgjkgghklhjgkghghjjkjhjkkggjkhjkhjjhhfhjhklkhkhjjklzzzyyyhjggjhgghhjhNWNWBWHWPPN222kkkgtytghjjh3434454}
curl_{\vec x}\vec v=\vec Q'=\vec Q,
\end{equation}
and
\begin{equation}
\label{MaxVacFull1ninshtrgravortghhghgjkgghklhjgkghghjjkjhjkkggjkhjkhjjhhfhjhklkhkhjjklzzzyyyhjggjhgghhjhNWNWBWHWPPN222kkkgtytghjjhghhg3434454}
\frac{d\vec W}{\partial t}(t)+\frac{\partial\vec v}{\partial
t}+d_{\vec x}\vec v\cdot\vec v=-c\,\vec R'=-c\,\vec R.
\end{equation}
Therefore, by \er{guigjgjffghguygjyfDDnw22mm3434},
\er{MaxVacFull1ninshtrgravortghhghgjkgghklhjgkghghjjkjhjkkggjkhjkhjjhhfhjhklkhkhjjklzzzyyyhjggjhgghhjhNWNWBWHWPPN222kkkgtytghjjh3434454}
and
\er{MaxVacFull1ninshtrgravortghhghgjkgghklhjgkghghjjkjhjkkggjkhjkhjjhhfhjhklkhkhjjklzzzyyyhjggjhgghhjhNWNWBWHWPPN222kkkgtytghjjhghhg3434454}
in the system $\bf(*)$ we have
\begin{equation}\label{guigjgjffghguygjyfDDnw22khhjhgg22mmhkkhhjjjg}
\begin{cases}
curl_{\vec x}\vec B=\frac{4\pi}{c}\vec
j+\frac{1}{c}\frac{\partial\vec D}{\partial t}-\frac{1}{c}curl_{\vec
x}\left(\vec v\times \vec
D\right),\\
div_{\vec x}\vec D=4\pi\rho,\\
curl_{\vec x}\vec D+\frac{1}{c}\frac{\partial\vec B}{\partial t}
-\frac{1}{c}\,curl_{\vec x}\left(\vec v\times\vec B\right)=0,\\
div_{\vec x}\vec B=0,\\
\frac{1}{c}\frac{d\vec W}{\partial t}(t)+\vec
R=-\frac{1}{c}\left(\frac{\partial\vec v}{\partial t}+d_{\vec x}\vec
v\cdot\vec v\right),\\
\vec Q=curl_{\vec x}\vec v,
\\
\frac{4\pi G}{c^2}\left(\mu\vec u-\mu\vec v+\frac{1}{4\pi c}\vec
D\times\vec
B-\frac{c}{4\pi G}\vec R\times\vec Q\right)=\\
(1+\beta)\,curl_{\vec x}\vec Q
-\frac{1}{c}\left(\frac{\partial\vec R}{\partial t}-curl_{\vec
x}\left\{\vec v\times\vec R\right\}+\left(div_{\vec x}\vec
R\right)\vec v\right),\\
\Div_{\vec x}\vec R=\frac{4\pi G}{c} (\mu+Q_0),\\
\frac{\partial Q_0}{\partial t}+div_{\vec x}\left(Q_0\vec
v\right)=-\Div_{\vec x}\left\{\frac{1}{4\pi c}\vec D\times\vec
B-\frac{c}{4\pi  G}\vec R\times\vec Q\right\}.
\end{cases}
\end{equation}
On the other hand, since the system $\bf(**)$ is inertial, the
quantity $\frac{d\vec W}{\partial t}(t)$, being generated by the
gravitational field from the far bodies, is insignificant
with respect to the quantity $c\vec R$ in the scale compatible to
the Earth size. Thus, we rewrite
\er{guigjgjffghguygjyfDDnw22khhjhgg22mmhkkhhjjjg} as:
\begin{equation}\label{guigjgjffghguygjyfDDnw22khhjhgg22mmhkkhhjjjgioo}
\begin{cases}
curl_{\vec x}\vec B=\frac{4\pi}{c}\vec
j+\frac{1}{c}\frac{\partial\vec D}{\partial t}-\frac{1}{c}curl_{\vec
x}\left(\vec v\times \vec
D\right),\\
div_{\vec x}\vec D=4\pi\rho,\\
curl_{\vec x}\vec D+\frac{1}{c}\frac{\partial\vec B}{\partial t}
-\frac{1}{c}\,curl_{\vec x}\left(\vec v\times\vec B\right)=0,\\
div_{\vec x}\vec B=0,\\
\vec R=-\frac{1}{c}\left(\frac{\partial\vec v}{\partial t}+d_{\vec
x}\vec
v\cdot\vec v\right),\\
\vec Q=curl_{\vec x}\vec v,
\\
\frac{4\pi G}{c^2}\left(\mu\vec u-\mu\vec v+\frac{1}{4\pi c}\vec
D\times\vec
B-\frac{c}{4\pi G}\vec R\times\vec Q\right)=\\
(1+\beta)\,curl_{\vec x}\vec Q
-\frac{1}{c}\left(\frac{\partial\vec R}{\partial t}-curl_{\vec
x}\left\{\vec v\times\vec R\right\}+\left(div_{\vec x}\vec
R\right)\vec v\right),\\
\Div_{\vec x}\vec R=\frac{4\pi G}{c} (\mu+Q_0),\\
\frac{\partial Q_0}{\partial t}+div_{\vec x}\left(Q_0\vec
v\right)=-\Div_{\vec x}\left\{\frac{1}{4\pi c}\vec D\times\vec
B-\frac{c}{4\pi  G}\vec R\times\vec Q\right\}.
\end{cases}
\end{equation}
Next, we can neglect all the far cosmic body masses expect of the
Earth itself and thus we can consider
\begin{equation}
\label{MaxVacFull1ninshtrgravortghhghgjkgghklhjgkghghjjkjhjkkggjkhjkhjjhhfhjhklkhkhjjklzzzyyyhjggjhgghhjhNWNWBWHWPPN222kkkjhjhjjhjhjhjhhjhj3444jkkyipi}
\mu(\vec x,t)=\mu_1\left(|\vec x|\right)\quad\text{and}\quad\vec
u(\vec x,t)=0 \quad\quad\forall\,(\vec x,t),
\end{equation}
where $\mu_1:=\mu_1\left(|\vec x|\right)$ is the inertial mass
density of the Earth which is assumed to be a radial function such
that
\begin{equation}
\label{MaxVacFull1ninshtrgravortghhghgjkgghklhjgkghghjjkjhjkkggjkhjkhjjhhfhjhklkhkhjjklzzzyyyhjggjhgghhjhNWNWBWHWPPN222kkkjhjhjjhjhjhjhhjhj3444jkkyipioyiyi}
\mu_1\left(|\vec x|\right)=0\quad\quad\text{if}\quad |\vec x|>r_0,
\end{equation}
where $r_0$ is the Earth radius. Moreover, we can neglect all the
electromagnetical masses and thus we simplify the equations for the
Gravity in \er{guigjgjffghguygjyfDDnw22khhjhgg22mmhkkhhjjjgioo} as:
\begin{equation}\label{guigjgjffghguygjyfDDnw22khhjhgg22mmhkkhhjjjgioojkjkggjg}
\begin{cases}
\vec R=-\frac{1}{c}\left(\frac{\partial\vec v}{\partial t}+d_{\vec
x}\vec
v\cdot\vec v\right),\\
\vec Q=curl_{\vec x}\vec v,
\\
\frac{4\pi G}{c^2}\left(-\mu_1\left(|\vec x|\right)\vec v-\frac{c}{4\pi G}\vec R\times\vec Q\right)=\\
(1+\beta)\,curl_{\vec x}\vec Q
-\frac{1}{c}\left(\frac{\partial\vec R}{\partial t}-curl_{\vec
x}\left\{\vec v\times\vec R\right\}+\left(div_{\vec x}\vec
R\right)\vec v\right),\\
\Div_{\vec x}\vec R=\frac{4\pi G}{c} \left(\mu_1\left(|\vec x|\right)+Q_0\right),\\
\frac{\partial Q_0}{\partial t}+div_{\vec x}\left(Q_0\vec
v\right)=\Div_{\vec x}\left\{\frac{c}{4\pi  G}\vec R\times\vec
Q\right\}.
\end{cases}
\end{equation}
Being in the system  $\bf(*)$ which is stationary with respect to
the center of the Earth we look for stationary (i.e. time
independent) solutions of
\er{guigjgjffghguygjyfDDnw22khhjhgg22mmhkkhhjjjgioojkjkggjg}. Thus
\er{guigjgjffghguygjyfDDnw22khhjhgg22mmhkkhhjjjgioojkjkggjg}
implies:
\begin{equation}\label{guigjgjffghguygjyfDDnw22khhjhgg22mmhkkhhjjjgioojkjkggjghjggjhh}
\begin{cases}
\vec R=-\frac{1}{c}\,d_{\vec x}\vec
v\cdot\vec v,\\
\vec Q=curl_{\vec x}\vec v,
\\
\frac{4\pi G}{c^2}\left(-\mu_1\left(|\vec x|\right)\vec v-\frac{c}{4\pi G}\vec R\times\vec Q\right)=\\
(1+\beta)\,curl_{\vec x}\vec Q
-\frac{1}{c}\left(-curl_{\vec x}\left\{\vec v\times\vec
R\right\}+\left(div_{\vec x}\vec
R\right)\vec v\right),\\
\Div_{\vec x}\vec R=\frac{4\pi G}{c} \left(\mu_1\left(|\vec x|\right)+Q_0\right),\\
div_{\vec x}\left(Q_0\vec v\right)=\Div_{\vec x}\left\{\frac{c}{4\pi
G}\vec R\times\vec Q\right\}.
\end{cases}
\end{equation}
On the other hand, by the symmetry considerations of the problem we
look for the solution of
\er{guigjgjffghguygjyfDDnw22khhjhgg22mmhkkhhjjjgioojkjkggjghjggjhh}
that satisfies $\vec v(\vec x)=\nabla_{\vec x}Z_0\left(|\vec
x|\right)$ where, again by the symmetry of the problem, the scalar
function $Z_0\left(|\vec x|\right)$ should be radial. In particular,
by
\er{guigjgjffghguygjyfDDnw22khhjhgg22mmhkkhhjjjgioojkjkggjghjggjhh}
we obtain
\begin{equation}\label{guigjgjffghguygjyfDDnjggjgjgj}
\vec Q=curl_{\vec x}\vec v=0
\end{equation}
and thus we simplify
\er{guigjgjffghguygjyfDDnw22khhjhgg22mmhkkhhjjjgioojkjkggjghjggjhh}
as:
\begin{equation}\label{guigjgjffghguygjyfDDnw22khhjhgg22mmhkkhhjjjgioojkjkggjghjggjhhuiiyyiyijkjkhkkh}
\begin{cases}
\vec R=-\frac{1}{c}\,d_{\vec x}\vec v\cdot\vec
v=-\frac{1}{c}\,\nabla_{\vec x}\left(\frac{1}{2}\left|\vec
v\right|^2\right),\\
\vec Q=curl_{\vec x}\vec v=0,
\\
\frac{4\pi G}{c^2}\,\mu_1\left(|\vec x|\right)\vec v=
\frac{1}{c}\left(-curl_{\vec x}\left\{\vec v\times\vec
R\right\}+\left(div_{\vec x}\vec
R\right)\vec v\right),\\
\Div_{\vec x}\vec R=\frac{4\pi G}{c}\,\mu_1\left(|\vec x|\right),\\
curl_{\vec x}\vec R=0,\\
Q_0=0.
\end{cases}
\end{equation}
In particular, since $\vec v=\nabla_{\vec x}Z_0\left(|\vec
x|\right)$ and $\vec R=-\frac{1}{c}\,\nabla_{\vec
x}\left(\frac{1}{2}\left|\vec v\right|^2\right)$ are both gradients
of radial functions, we have $\vec v\times\vec R=0$ and thus, we
further simplify
\er{guigjgjffghguygjyfDDnw22khhjhgg22mmhkkhhjjjgioojkjkggjghjggjhhuiiyyiyijkjkhkkh}
as:
\begin{equation}\label{guigjgjffghguygjyfDDnw22khhjhgg22mmhkkhhjjjgioojkjkggjghjggjhhuiiyyiyijkjkhkkhuu}
\begin{cases}
\vec R=-\frac{1}{c}\,\nabla_{\vec x}\left(\frac{1}{2}\left|\vec
v\right|^2\right),\\
\vec Q=curl_{\vec x}\vec v=0,
\\
\Div_{\vec x}\vec R=\frac{4\pi G}{c}\,\mu_1\left(|\vec x|\right),\\
curl_{\vec x}\vec R=0,\\
Q_0=0.
\end{cases}
\end{equation}
However,
\er{guigjgjffghguygjyfDDnw22khhjhgg22mmhkkhhjjjgioojkjkggjghjggjhhuiiyyiyijkjkhkkhuu}
is equivalent to the following:
\begin{equation}\label{guigjgjffghguygjyfDDnw22khhjhgg22mmhkkhhjjjgioojkjkggjghjggjhhuiiyyiyijkjkhkkhuuhhhhhjgjikuo}
\begin{cases}
\Delta_{\vec x}\left(-\frac{1}{2}\left|\vec
v\right|^2\right)=4\pi G\,\mu_1\left(|\vec x|\right),\\
curl_{\vec x}\vec v=0,\\
\vec R=-\frac{1}{c}\,\nabla_{\vec x}\left(\frac{1}{2}\left|\vec
v\right|^2\right),\\
\vec Q=0,
\\
Q_0=0.
\end{cases}
\end{equation}
Therefore, denoting
\begin{equation}\label{gjfghdhgdfjfjfgukiufirfiufuy}
\Phi_1:=-\frac{1}{2}\left|\vec v\right|^2
\end{equation}
we rewrite
\er{guigjgjffghguygjyfDDnw22khhjhgg22mmhkkhhjjjgioojkjkggjghjggjhhuiiyyiyijkjkhkkhuuhhhhhjgjikuo}
as:
\begin{equation}\label{guigjgjffghguygjyfDDnw22khhjhgg22mmhkkhhjjjgioojkjkggjghjggjhhuiiyyiyijkjkhkkhuuhhhhhjgjikuojhjhjhjguyttty}
\begin{cases}
\Delta_{\vec x}\Phi_1=4\pi G\,\mu_1\left(|\vec x|\right),\\
\Phi_1:=-\frac{1}{2}\left|\vec v\right|^2,
\\
curl_{\vec x}\vec v=0,\\
\vec R=-\frac{1}{c}\,\nabla_{\vec x}\left(\frac{1}{2}\left|\vec
v\right|^2\right),\\
\vec Q=0,
\\
Q_0=0,
\end{cases}
\end{equation}
where the radial scalar field $\Phi_1:=\Phi_1\left(|\vec x|\right)$,
coincides with the usual Newtonian potential of the Earth, and
outside of the Earth surface it is known that $\Phi_1(\vec
x)=-\frac{Gm_0}{|\vec x|}$, where $m_0$ is the total mass of the
Earth (we have $m_0=\iiint_{|\vec x|\leq r_0}\mu_1\left(|\vec
x|\right)d\vec x$). Thus, since there exists a scalar radial field
$Z_0(\vec x)$ such that $\vec v(\vec x)=\nabla_{\vec
x}Z_0\left(|\vec x|\right)$
by
\er{guigjgjffghguygjyfDDnw22khhjhgg22mmhkkhhjjjgioojkjkggjghjggjhhuiiyyiyijkjkhkkhuuhhhhhjgjikuojhjhjhjguyttty}
we obtain
\begin{equation}
\label{MaxVacFull1ninshtrgravortghhghgjkgghklhjgkghghjjkjhjkkggjkhjkhjjhhfhjhklkhkhjjklzzzyyyhjggjhgghhjhNWNWBWHWPPN222kkkgtytghjjhpoppkkkhhhk34}
\left|\frac{dZ_0}{d(|\vec x|)}(|\vec x|)\right|= \sqrt{-2\Phi_1(\vec
x)},
\end{equation}
that implies either
\begin{equation}
\label{MaxVacFull1ninshtrgravortghhghgjkgghklhjgkghghjjkjhjkkggjkhjkhjjhhfhjhklkhkhjjklzzzyyyhjggjhgghhjhNWNWBWHWPPN222kkkgtytghjjhpoppkkkhhhkhjhj34}
\vec v(\vec x)= \frac{\sqrt{-2\Phi_1(|\vec x|)}}{|\vec x|}\vec x,
\end{equation}
or
\begin{equation}
\label{MaxVacFull1ninshtrgravortghhghgjkgghklhjgkghghjjkjhjkkggjkhjkhjjhhfhjhklkhkhjjklzzzyyyhjggjhgghhjhNWNWBWHWPPN222kkkgtytghjjhpoppkkkhhhkhjhjiuu34}
\vec v(\vec x)= -\frac{\sqrt{-2\Phi_1(|\vec x|)}}{|\vec x|}\vec x,
\end{equation}
exactly as in the case of the usual Newtonian gravity. In
particular, on the Earth surface we have:
\begin{equation}
\label{MaxVacFull1ninshtrgravortghhghgjkgghklhjgkghghjjkjhjkkggjkhjkhjjhhfhjhklkhkhjjklzzzyyyhjggjhgghhjhNWNWBWHWPPN222kkkgtytghjjhpoppkkkhhhkhjhjiuuokokok34}
|\vec v|=\sqrt{\frac{2Gm_0}{r_0}},
\end{equation}
where $r_0$ is the Earth radius and $m_0$ is the Earth mass, i.e.
the absolute value of the vectorial gravitational potential on the
Earth surface approximately equals to the escape velocity and its
direction is normal to the Earth, either downward or upward.
\end{remark}

\section{Covariant formulation of the physical laws in the
four-dimensional non-relativistic space-time}\label{CVFRM}
\subsection{Four-vectors, four-covectors and tensors in the
four-dimensional non-relativistic space-time} First of all we would
like to remind the definitions of the vectors, covectors and
covariant and contravariant tensors of second order in
$\mathbb{R}^4$.
\begin{definition}
Given $\mathcal{S}$, that is a certain subgroup of the group of all
smooth non-degenerate invertible transformations from $\mathbb{R}^4$
onto $\mathbb{R}^4$ having the form
\begin{equation}\label{fgjfjhgghyuyyu}
\begin{cases}
x'^0=f^{(0)}(x^0,x^1,x^2,x^3),\\
x'^1=f^{(1)}(x^0,x^1,x^2,x^3),\\
x'^2=f^{(2)}(x^0,x^1,x^2,x^3),\\
x'^3=f^{(3)}(x^0,x^1,x^2,x^3),
\end{cases}
\end{equation}
we say that a one-component field $a:=a(x^0,x^1,x^2,x^3)$ is a
scalar field on the group $\mathcal{S}$, if under the coordinate
transformation in the group $\mathcal{S}$ of the form
\er{fgjfjhgghyuyyu} this field transforms as:
\begin{equation}\label{fgjfjhgghhgjgihhi}
a'=a.
\end{equation}
Next we say that a four-component field $(a^0,a^1,a^2,a^3)$ is a
four-vector field on the group $\mathcal{S}$, if under the
coordinate transformation in the group $\mathcal{S}$ of the form
\er{fgjfjhgghyuyyu} every of four components of this field
transforms as:
\begin{equation}\label{fgjfjhgghhgjg}
a'^j=\sum_{k=0}^{3}\frac{\partial f^{(j)}}{\partial
x^k}a^k\quad\quad\forall j=0,1,2,3.
\end{equation}
Next we say that a four-component field $(a_0,a_1,a_2,a_3)$ is a
four-covector field on the group $\mathcal{S}$, if under the
coordinate transformation in the group $\mathcal{S}$ of the form
\er{fgjfjhgghyuyyu} every of four components of this field
transforms as:
\begin{equation}\label{fgjfjhgghhgjghjhj}
a_j=\sum_{k=0}^{3}\frac{\partial f^{(k)}}{\partial
x^j}a'_k\quad\quad\forall j=0,1,2,3.
\end{equation}
Furthermore, we say that a $16$-component field
$\{a_{mn}\}_{m,n=0,1,2,3}$ is a two times covariant tensor field on
the group $\mathcal{S}$, if under the coordinate transformation in
the group $\mathcal{S}$ of the form \er{fgjfjhgghyuyyu} every of
$16$ components of this field transforms as:
\begin{equation}\label{fgjfjhgghhgjghjhjkkkkjjk}
a_{mn}=\sum_{j=0}^{3}\sum_{k=0}^{3}\frac{\partial f^{(k)}}{\partial
x^m}\frac{\partial f^{(j)}}{\partial x^n}a'_{kj}\quad\quad\forall\,
m,n=0,1,2,3.
\end{equation}
Next we say that a $16$-component field $\{a^{mn}\}_{m,n=0,1,2,3}$
is a two times contravariant tensor field on the group
$\mathcal{S}$, if under the coordinate transformation in the group
$\mathcal{S}$ of the form \er{fgjfjhgghyuyyu} every of $16$
components of this field transforms as:
\begin{equation}\label{fgjfjhgghhgjghjhjkkkkgggh}
a'^{mn}=\sum_{j=0}^{3}\sum_{k=0}^{3}\frac{\partial f^{(m)}}{\partial
x^k}\frac{\partial f^{(n)}}{\partial x^j}a^{kj}\quad\quad\forall\,
m,n=0,1,2,3.
\end{equation}
Then it is well known that for every two four-vectors
$(a^0,a^1,a^2,a^3)$ and $(b^0,b^1,b^2,b^3)$ on $\mathcal{S}$, the
$16$-component field $\{c^{mn}\}_{m,n=0,1,2,3}$, defined in every
coordinate system by
\begin{equation}\label{fgjfjhgghhgjghjhjkkkkgjghghjljl}
c^{mn}:=a^mb^n\quad\quad\forall\, m,n=0,1,2,3,
\end{equation}
is a two times contravariant tensor on $\mathcal{S}$. Moreover, for
every two four-covectors $(a_0,a_1,a_2,a_3)$ and $(b_0,b_1,b_2,b_3)$
on $\mathcal{S}$, the $16$-component field
$\{c_{mn}\}_{m,n=0,1,2,3}$, defined in every coordinate system by
\begin{equation}\label{fgjfjhgghhgjghjhjkkkkgjghghkkkj}
c_{mn}:=a_mb_n\quad\quad\forall\, m,n=0,1,2,3,
\end{equation}
is a two times covariant tensor on $\mathcal{S}$. It is also well
known that if $\{a^{mn}\}_{m,n=0,1,2,3}$ is a two times
contravariant tensor field on the group $\mathcal{S}$ and if a
$16$-component field $\{b_{mn}\}_{m,n=0,1,2,3}$ satisfies
\begin{equation}\label{fgjfjhgghhgjghjhjkkkkgjghghuiiiu}
\sum_{k=0}^{3}a^{mk}b_{kn}=\begin{cases}
1\quad\text{if}\quad m=n\\
0\quad\text{if}\quad m\neq n
\end{cases}\quad\quad\forall\, m,n=0,1,2,3,
\end{equation}
then $\{b_{mn}\}_{m,n=0,1,2,3}$ is a two times covariant tensor on
$\mathcal{S}$. Next it is well known that, given a four-covector
$(a_0,a_1,a_2,a_3)$ a four-vector $(b^0,b^1,b^2,b^3)$, a two times
covariant tensor $\{c_{mn}\}_{m,n=0,1,2,3}$ and a two times
contravariant tensor $\{d^{mn}\}_{m,n=0,1,2,3}$ on the group
$\mathcal{S}$, the quantities
\begin{equation}\label{fgjfjhgghhgjghjhjkkkkgjghghuiiiukljk}
\sum_{k=0}^{3}a_kb^k\quad\text{and}\quad
\sum_{m=0}^{3}\sum_{n=0}^{3}c_{mn}d^{mn}
\end{equation}
are scalars on $\mathcal{S}$, the four-component fields defined by
\begin{equation}\label{fgjfjhgghhgjghjhjkkkkgjghghuiiiulkkj}
\Big\{\sum_{k=0}^{3}d^{mk}a_{k}\Big\}_{m=0,1,2,3}\quad\text{and}\quad
\Big\{\sum_{k=0}^{3}c_{mk}b^{k}\Big\}_{m=0,1,2,3}
\end{equation}
are four-vector and four-covector on $\mathcal{S}$ and moreover,
$16$-component fields $\{\hat c^{mn}\}_{m,n=0,1,2,3}$ and $\{\hat
d_{mn}\}_{m,n=0,1,2,3}$ defined by
\begin{equation}\label{fgjfjhgghhgjghjhjkkkkgjghghuiiiulkkjlkkl}
\hat
c^{mn}:=\sum_{k=0}^{3}\sum_{j=0}^{3}d^{mj}d^{nk}c_{jk}\quad\text{and}\quad
\hat
d_{mn}:=\sum_{j=0}^{3}\sum_{k=0}^{3}c_{mj}c_{nk}d^{jk}\quad\quad\forall\,
m,n=0,1,2,3,
\end{equation}
are two times contravariant and two times covariant tensors on
$\mathcal{S}$. Next, it is also well known that given a two times
covariant tensor $\{c_{mn}\}_{m,n=0,1,2,3}$ and a two times
contravariant tensor $\{d^{mn}\}_{m,n=0,1,2,3}$ on the group
$\mathcal{S}$ the $16$-component fields $\{c_{nm}\}_{m,n=0,1,2,3}$
and $\{d^{nm}\}_{m,n=0,1,2,3}$ are also two times covariant and two
times contravariant tensors on $\mathcal{S}$. Finally, it is well
known that, if $a:=a(x^0,x^1,x^2,x^3)$ is a scalar field on the
group $\mathcal{S}$, then the four-component field
$(w_0,w_1,w_2,w_3)$ defined by:
\begin{equation}\label{fgjfjhgghhgjghjhjkkkkgjghghuiiiulkkjlkklplikkl}
w_j:=\frac{\partial a}{\partial x^j}\quad\quad\forall\,j=0,1,2,3,
\end{equation}
is a \underline{four-covector} field on the group $\mathcal{S}$.
\end{definition}
Next consider the four-dimensional space-time $\mathbb{R}^4$, such
that for every point in space $\vec x=(x_1,x_2,x_3)\in\mathbb{R}^3$
and every instant of time $t$ we correspond the point
$(x^0,x^1,x^2,x^3)\in\mathbb{R}^4$ that has the form:
\begin{equation}\label{fgjfjhggh}
(x^0,x^1,x^2,x^3):=(ct,x_1,x_2,x_3)=\left(ct,\vec x\right),
\end{equation}
where $c$ is the universal constant in Maxwell equations for vacuum.
In this space we denote by $\mathcal{S}_0$, the subgroup of the
group of smooth non-degenerate invertible mappings, containing
transformations of the form
\begin{equation}\label{noninchgravortbstrjgghguittu2intrrrZZygjyghhj}
\begin{cases}
x'^0=x^0
\\
x'^j=\sum\limits_{k=1}^{3}A_{jk}\left(\frac{x^0}{c}\right)\,x_k+z_j\left(\frac{x^0}{c}\right)\quad\forall
j=1,2,3,
\end{cases}
\end{equation}
where $$\left\{A_{jk}(t)\right\}_{j,k=1,2,3}=A(t):\mathbb{R}\to
SO(3)$$ is a rotation, smoothly dependent on $t$ and
$$\left(z_1(t),z_2(t),z_3(t)\right)=\vec
z(t):\mathbb{R}\to\mathbb{R}^3$$ also smoothly dependent on $t$.
Then in the terms of time $t$ and three-dimensional space we rewrite
\er{noninchgravortbstrjgghguittu2intrrrZZygjyghhj} as:
\begin{equation}\label{noninchgravortbstrjgghguittu2intrrrZZygjyg}
\begin{cases}
\vec x'=A(t)\cdot\vec x+\vec z(t),\\
t'=t,
\end{cases}
\end{equation}
where $A(t)\in SO(3)$ is a rotation. I.e. the group $\mathcal{S}_0$
represents all transformations of cartesian non-inertial coordinate
systems in the non-relativistic space-time. It can be easily checked
by trivial calculations that $\mathcal{S}_0$ is indeed a group, i.e.
for every two transformations $f,g\in\mathcal{S}_0$ the composition
$g\circ f$ and the inverse transformation $f^{(-1)}$ are also
contained in $\mathcal{S}_0$, thats mean that they also have a form
of \er{noninchgravortbstrjgghguittu2intrrrZZygjyghhj}. Next assume
that a four-covector $(a_0,a_1,a_2,a_3)$ and a four-vector
$(b^0,b^1,b^2,b^3)$ on the group $\mathcal{S}_0$ are given. Then by
inserting \er{noninchgravortbstrjgghguittu2intrrrZZygjyghhj} into
\er{fgjfjhgghhgjg} and \er{fgjfjhgghhgjghjhj} we obtain the
following laws of transformations under the acting in the group
$\mathcal{S}_0$:
\begin{equation}\label{fgjfjhgghhgjghjhjijhoj}
\begin{cases}
a_0=a'_0+\sum_{k=1}^{3}\frac{1}{c}\left(\sum\limits_{j=1}^{3}\frac{dA_{kj}}{dt}\left(\frac{x^0}{c}\right)\,x_j+\frac{d
z_k}{dt}\left(\frac{x^0}{c}\right)\right)\,a'_k
\\
a_j=\sum_{k=1}^{3}A_{kj}\left(\frac{x^0}{c}\right)\,a'_k\quad\quad\forall
j=1,2,3,
\end{cases}
\end{equation}
and
\begin{equation}\label{fgjfjhgghhgjgiuouoiuu}
\begin{cases}
b'^0=b^0
\\
b'^j=\frac{1}{c}\left(\sum\limits_{k=1}^{3}\frac{dA_{jk}}{dt}\left(\frac{x^0}{c}\right)\,x_k+\frac{d
z_j}{dt}\left(\frac{x^0}{c}\right)\right)\,b^0+\sum_{k=1}^{3}A_{jk}\left(\frac{x^0}{c}\right)\,b^k\quad\quad\forall
j=1,2,3.
\end{cases}
\end{equation}
In particular, since $A(t)\in SO(3)$ and thus
\begin{equation}\label{fgjfjhgghhgjgiuouoiuu1}
\sum_{j=1}^{3}A_{mj}(t)A_{nj}(t)=\begin{cases}1\quad\text{if}\quad m=n\\
0\quad\text{if}\quad m\neq n
\end{cases}
\quad\quad\quad\forall m,n=1,2,3,
\end{equation}
by \er{fgjfjhgghhgjghjhjijhoj} we deduce:
\begin{equation}\label{fgjfjhgghhgjghjhjijhojpiii}
\begin{cases}
a'_0=a_0-\sum_{k=1}^{3}\frac{1}{c}\left(\sum\limits_{j=1}^{3}\frac{dA_{kj}}{dt}\left(\frac{x^0}{c}\right)\,x_j+\frac{d
z_k}{dt}\left(\frac{x^0}{c}\right)\right)\left(\sum_{j=1}^{3}A_{kj}\left(\frac{x^0}{c}\right)\,a_j\right)
\\
a'_k=\sum_{j=1}^{3}A_{kj}\left(\frac{x^0}{c}\right)\,a_j\quad\quad\forall
k=1,2,3.
\end{cases}
\end{equation}
So, by \er{fgjfjhgghhgjghjhjijhojpiii} and
\er{fgjfjhgghhgjgiuouoiuu} we obtained the following laws of
transformation of four-covectors and four-vectors in the group
$\mathcal{S}_0$, i.e. under the change of non-inertial cartesian
coordinate systems:
\begin{equation}\label{fgjfjhgghhgjghjhjijhojpiiihjhj}
\begin{cases}
a'_0=a_0-\sum_{k=1}^{3}\frac{1}{c}\left(\sum\limits_{j=1}^{3}\frac{dA_{kj}}{dt}\left(\frac{x^0}{c}\right)\,x_j+\frac{d
z_k}{dt}\left(\frac{x^0}{c}\right)\right)\left(\sum_{j=1}^{3}A_{kj}\left(\frac{x^0}{c}\right)\,a_j\right)
\\
a'_k=\sum_{j=1}^{3}A_{kj}\left(\frac{x^0}{c}\right)\,a_j\quad\quad\forall
k=1,2,3,
\end{cases}
\end{equation}
and
\begin{equation}\label{fgjfjhgghhgjgiuouoiuujkjkjk}
\begin{cases}
b'^0=b^0
\\
b'^j=\frac{1}{c}\left(\sum\limits_{k=1}^{3}\frac{dA_{jk}}{dt}\left(\frac{x^0}{c}\right)\,x_k+\frac{d
z_j}{dt}\left(\frac{x^0}{c}\right)\right)\,b^0+\sum_{k=1}^{3}A_{jk}\left(\frac{x^0}{c}\right)\,b^k\quad\quad\forall
j=1,2,3.
\end{cases}
\end{equation}
%
%
%
%
%
%
Therefore, if we denote the four-vector $(b^0,b^1,b^2,b^3)$ and the
four-covector $(a_0,a_1,a_2,a_3)$ on the group $\mathcal{S}_0$ as:
\begin{equation}\label{fgjfjhgghhgjghjhjijhojihjhjjijhjjj}
\begin{cases}
(b^0,b^1,b^2,b^3)=\left(\sigma,\frac{1}{c}\vec
b\right)\quad\text{where}\quad
\sigma:=b^0\;\;\text{and}\;\;\vec b:=c(b^1,b^2,b^3)\in\mathbb{R}^3,\\
(a_0,a_1,a_2,a_3)=(\psi,-\vec a)\quad\text{where}\quad
\psi:=a_0\;\;\text{and}\;\;\vec a:=-(a_1,a_2,a_3)\in\mathbb{R}^3,
\end{cases}
\end{equation}
then by \er{fgjfjhgghhgjghjhjijhojpiiihjhj} and
\er{fgjfjhgghhgjgiuouoiuujkjkjk} in the terms of time $t$ and three
dimensional space $\vec x$, we obtain the following laws of
transformations of $\sigma$, $\vec b$, $\psi$ and $\vec a$ under the
change of non-inertial cartesian coordinate system:
\begin{equation}\label{fgjfjhgghhgjgiuouoiuujkjkjkojkoiu}
\begin{cases}
\sigma'=\sigma
\\
\vec b'=A\left(t\right)\cdot\vec
b+\left(\frac{dA}{dt}\left(t\right)\cdot\vec x+\frac{d \vec
z}{dt}\left(t\right)\right)\sigma,
\end{cases}
\end{equation}
and
\begin{equation}\label{fgjfjhgghhgjghjhjijhojpiiihjhjuiouj}
\begin{cases}
\psi'=\psi+\frac{1}{c}\left(\frac{dA}{dt}\left(t\right)\cdot\vec
x+\frac{d \vec z}{dt}\left(t\right)\right)\cdot\left(A(t)\cdot\vec
a\right)
\\
\vec a'=A(t)\cdot\vec a.
\end{cases}
\end{equation}
In particular, if $\sigma:=b^0$ is the first coordinate of an
arbitrary four-vector $(b^0,b^1,b^2,b^3)$ on the group
$\mathcal{S}_0$, then $\sigma$ is a proper scalar field in the
frames of Definition \ref{bggghghgj}. Moreover, if $\vec
a:=-(a_1,a_2,a_3)$, where $a_1,a_2,a_3$ are the last three
coordinates of an arbitrary four-covector $(a_0,a_1,a_2,a_3)$ on the
group $\mathcal{S}_0$, then $\vec a$ is a proper vector field in the
frames of Definition \ref{bggghghgj}.

Next, since by Definition \ref{bggghghgj} every three-dimensional
speed-like vector field $\vec u$ transforms under the change of
non-inertial cartesian coordinate system as:
\begin{equation}
\label{NoIn3redPPN'}\vec u'=A(t)\cdot \vec
u+\frac{dA}{dt}\left(t\right)\cdot\vec x+\frac{d \vec
z}{dt}\left(t\right),
\end{equation}
by comparing \er{NoIn3redPPN'} with
\er{fgjfjhgghhgjgiuouoiuujkjkjkojkoiu} we deduce that for every
speed-like vector field $\vec u$ the four-component field
$(u^0,u^1,u^2,u^3)$ defined by
\begin{equation}\label{fgjfjhgghhgjghjhjijhojihjhjjijhjjjjjuii}
(u^0,u^1,u^2,u^3):=\left(1,\frac{1}{c}\vec
u\right)\quad\text{where}\quad
u^0=1\;\;\text{and}\;\;(u^1,u^2,u^3)=\frac{1}{c}\vec
u\in\mathbb{R}^3,
\end{equation}
is a four-vector field on the group $\mathcal{S}_0$. We call such
four-vectors by the name vectors of type $1$. In particular, if
$\vec u$ is the velocity field, then the quantity defined by
\er{fgjfjhgghhgjghjhjijhojihjhjjijhjjjjjuii} is a a four-vector
field on the group $\mathcal{S}_0$ that we call the four-dimensional
speed. Regarding the field of velocity $\vec u$ we also can give a
different argumentation that the four-component field
$(u^0,u^1,u^2,u^3)$ defined by
\er{fgjfjhgghhgjghjhjijhojihjhjjijhjjjjjuii} is a four-vector field
on the group $\mathcal{S}_0$: indeed it is well known from Tensor
Analysys that if $\left(x^0(s),x^1(s),x^2(s),x^3(s)\right)$ is a
curve in $\mathbb{R}^4$, parameterized by some scalar parameter $s$,
then the four-component field
$\left(\frac{dx^0}{ds}(s),\frac{dx^1}{ds}(s),\frac{dx^2}{ds}(s),\frac{dx^3}{ds}(s)\right)$
is a four-vector field on an arbitrary group $\mathcal{S}$ and, in
particular, on the group $\mathcal{S}_0$. Thus, if $\vec
r(t)=\left(r_1(t),r_2(t),r_3(t)\right)$ is a three-dimensional
trajectory of the motion of some particle, parameterized by the
global time $t$, then if we consider a curve
$\frac{1}{c}\left(ct,r_1(t),r_2(t),r_3(t)\right)$ in $\mathbb{R}^4$,
parameterized by the global time $t$, then the four-component field:
\begin{equation}\label{fgjfjhgghhgjghjhjijhojihjhjjijhjjjjjuiijhjhh}
\left(1,\frac{1}{c}\frac{d\vec r}{dt}(t)\right):=
\left(1,\frac{1}{c}\frac{dr_1}{dt}(t),\frac{1}{c}\frac{dr_2}{dt}(t),\frac{1}{c}\frac{dr_3}{dt}(t)\right)
\end{equation}
is a four-vector field on the group $\mathcal{S}_0$.

 Similarly, if $\vec v$ is the vectorial gravitational
potential, then since $\vec v$ is a speed-like vector field, the
four-component field $(v^0,v^1,v^2,v^3)$ defined by
\begin{equation}\label{fgjfjhgghhgjghjhjijhojihjhjjijhjjjjjuiijjjk}
(v^0,v^1,v^2,v^3):=\left(1,\frac{1}{c}\vec
v\right)\quad\text{where}\quad
v^0=1\;\;\text{and}\;\;(v^1,v^2,v^3)=\frac{1}{c}\vec v,
\end{equation}
is also a four-vector field on the group $\mathcal{S}_0$ that we
call the four-dimensional gravitational potential. In the same way,
if $\vec k$ is the vectorial potential of inertia, defined by
Definition \ref{jjhgyftdtd}, then since $\vec k$ is a speed-like
vector field, the four-component field $(k^0,k^1,k^2,k^3)$ defined
by
\begin{equation}\label{fgjfjhgghhgjghjhjijhojihjhjjijhjjjjjuiijjjkjkkhhkkh}
(k^0,k^1,k^2,k^3):=\left(1,\frac{1}{c}\vec
k\right)\quad\text{where}\quad
k^0=1\;\;\text{and}\;\;(k^1,k^2,k^3)=\frac{1}{c}\vec k,
\end{equation}
is also a four-vector field on the group $\mathcal{S}_0$ that we
call the four-dimensional potential of inertia.

Moreover, by \er{fgjfjhgghhgjgiuouoiuujkjkjkojkoiu}, for every speed
like vector field $\vec u$ and every proper scalar field $\sigma$
the four-component field $(b^0,b^1,b^2,b^3)$ defined by
\begin{equation}\label{fgjfjhgghhgjghjhjijhojihjhjjijhjjjjjuiijjjkl}
(b^0,b^1,b^2,b^3):=\left(\sigma,\frac{\sigma}{c}\vec
u\right)\quad\text{where}\quad
b^0=\sigma\;\;\text{and}\;\;(b^1,b^2,b^3)=\frac{\sigma}{c}\vec u,
\end{equation}
is also a four-vector field on the group $\mathcal{S}_0$. In
particular, if we consider the field of four-dimensional moment of a
particle $(p^0,p^1,p^2,p^3)$ defined by
\begin{equation}\label{fgjfjhgghhgjghjhjijhojihjhjjijhjjjjjuiijjjklihh}
(p^0,p^1,p^2,p^3):=\left(m,\frac{1}{c}(m\vec
u)\right)\quad\text{where}\quad
p^0=m\;\;\text{and}\;\;(p^1,p^2,p^3)=\frac{1}{c}(m\vec u),
\end{equation}
where $m$ is the mass of the particle and $\vec u$ is the velocity
of the particle, then $(p^0,p^1,p^2,p^3)$ is also a four-vector on
the group $\mathcal{S}_0$. Moreover, by comparing \er{NoIn4redPPN}
and \er{NoIn6redPPN} with \er{fgjfjhgghhgjgiuouoiuujkjkjkojkoiu} we
deduce that if we consider the field of four-dimensional electric
current $(j^0,j^1,j^2,j^3)$ defined by
\begin{equation}\label{fgjfjhgghhgjghjhjijhojihjhjjijhjjjjjuiijjjklihhojjjo}
(j^0,j^1,j^2,j^3):=\left(\rho,\frac{1}{c}\,\vec
j\right)\quad\text{where}\quad
j^0=\rho\;\;\text{and}\;\;(j^1,j^2,j^3)=\frac{1}{c}\,\vec j,
\end{equation}
where $\rho$ is the electric charge density and $\vec j$ is the
electric current density, then $(j^0,j^1,j^2,j^3)$ is also a
four-vector on the group $\mathcal{S}_0$.

On the other hand, for every proper three-dimensional vector field
$\vec G$ that satisfies due to Definition \ref{bggghghgj}:
\begin{equation}
\label{NoIn1redPPN'}\vec G'=A(t)\cdot\vec G,
\end{equation}
by comparing \er{NoIn1redPPN'} with
\er{fgjfjhgghhgjgiuouoiuujkjkjkojkoiu} we deduce that the
four-component field $(G^0,G^1,G^2,G^3)$ defined by
\begin{equation}\label{fgjfjhgghhgjghjhjijhojihjhjjijhjjjjjuiiklkllo;}
(G^0,G^1,G^2,G^3):=\left(0,\vec G\right)\quad\text{where}\quad
G^0=0\;\;\text{and}\;\;(G^1,G^2,G^3)=\vec G,
\end{equation}
is also a four-vector field on the group $\mathcal{S}_0$.  We call
such four-vectors by the name vectors of type $0$.

Next, since by
\er{vhfffngghhjghhgPPNghghghutghfflklhjkjhjhjjgjkghhj} the scalar
electromagnetic potential $\Psi$ and the vector electromagnetic
potential $\vec A$, under the change of non-inertial cartesian
coordinate system transform as:
\begin{equation}\label{vhfffngghhjghhgPPNghghghutghfflklhjkjhjhjjgjkghhjhhjhjkhjghlkk}
\begin{cases}
\Psi'=
\Psi+\frac{1}{c}\left(\frac{dA}{dt}(t)\cdot\vec x+\frac{d\vec
z}{dt}(t)\right)\cdot\left(A(t)\cdot\vec A\right)
\\
\vec A'=A(t)\cdot \vec A,
\end{cases}
\end{equation}
by comparing
\er{vhfffngghhjghhgPPNghghghutghfflklhjkjhjhjjgjkghhjhhjhjkhjghlkk}
with \er{fgjfjhgghhgjghjhjijhojpiiihjhjuiouj} we deduce that the
four-component field $(A_0,A_1,A_2,A_3)$ defined as
\begin{equation}\label{fgjfjhgghhgjghjhjijhojihjhjjijhjjjljljpk}
(A_0,A_1,A_2,A_3)=(\Psi,-\vec A)\quad\text{where}\quad
A_0=\Psi\;\;\text{and}\;\;(A_1,A_2,A_3)=-\vec A,
\end{equation}
is a four-\underline{covector} field on the group $\mathcal{S}_0$.
We call this  four-covector field by the name four dimensional
electromagnetic potential. Next, since $(A_0,A_1,A_2,A_3)$ is a
four-covector field on the group $\mathcal{S}_0$, then it is well
known from the tensor analysis that the $16$-component field
$\{F_{ij}\}_{0\leq i,j\leq 3}$ defined in every non-inertial
cartesian coordinate system by
\begin{equation}\label{huohuioy89gjjhjffffff3478zzrrZZZhjhhjhhjjhhffGGhjjh}
F_{ij}:=\frac{\partial A_j}{\partial x^i}-\frac{\partial
A_i}{\partial x^j}\quad\quad\forall\, i,j=0,1,2,3\,,
\end{equation}
is an antisymmetric two times covariant tensor field on the group
$\mathcal{S}_0$, which we call the covariant tensor of the
electromagnetic field. In particular, by inserting
\er{fgjfjhgghhgjghjhjijhojihjhjjijhjjjljljpk} and \er{fgjfjhggh}
into \er{huohuioy89gjjhjffffff3478zzrrZZZhjhhjhhjjhhffGGhjjh} we
deduce:
\begin{equation}\label{huohuioy89gjjhjffffff3478zzrrZZZhjhhjhhjjhhffGGhjjhiuui}
\begin{cases}
F_{00}=0\\ F_{0j}=-F_{j0}=-\frac{1}{c}\frac{\partial(-A_j)}{\partial
t}-\frac{\partial \Psi}{\partial
x^j}\quad\quad\forall\, j=1,2,3\\
F_{jj}=0\quad\quad\forall\, j=1,2,3
\\
F_{ij}=-F_{ji}=\frac{\partial (-A_i)}{\partial x^j}-\frac{\partial
(-A_j)}{\partial x^i}\quad\quad\forall\, i\neq j=1,2,3\,,
\end{cases}
\end{equation}
Thus if as in \er{MaxVacFull1bjkgjhjhgjgjgkjfhjfdghghligioiuittrPPN}
we denote:
\begin{equation}\label{MaxVacFull1bjkgjhjhgjgjgkjfhjfdghghligioiuittrPPNkkk}
\begin{cases}
\vec B:= curl_{\vec x} \vec A,\\
\vec E:=-\nabla_{\vec x}\Psi-\frac{1}{c}\frac{\partial\vec
A}{\partial t},
\end{cases}
\end{equation}
then denoting  $\vec E:=(E_1,E_2,E_3)$ and $\vec B:=(B_1,B_2,B_3)$,
by \er{MaxVacFull1bjkgjhjhgjgjgkjfhjfdghghligioiuittrPPNkkk}
we rewrite
\er{huohuioy89gjjhjffffff3478zzrrZZZhjhhjhhjjhhffGGhjjhiuui} as:
\begin{equation}\label{huohuioy89gjjhjffffff3478zzrrZZZhjhhjhhjjhhffGGhjjhiuuijkjjk}
\begin{cases}
F_{00}=0\\ F_{0j}=-F_{j0}=E_j\quad\quad\forall\, j=1,2,3\\
F_{jj}=0\quad\quad\forall\, j=1,2,3
\\
F_{12}=-F_{21}=-B_3
\\
F_{13}=-F_{31}=B_2
\\
F_{23}=-F_{32}=-B_1\,.
\end{cases}
\end{equation}

Next assume that $T:=\left\{T_{ij}\right\}_{i,j=1,2,3}\in
\R^{3\times 3}$ is a $9$-component proper matrix valued field,
which, being a proper matrix field, by Definition \ref{bggghghgj}
satisfies:
\begin{equation}\label{uguyytfddddgghjjghjjjihohjjk}
T'=A(t)\cdot T\cdot A^T(t)=A(t)\cdot T\cdot
\left\{A(t)\right\}^{-1}.
\end{equation}
Next consider a $16$-component field $\{\mathcal{T}^{ij}\}_{0\leq
i,j\leq 3}$ defined in every non-inertial cartesian coordinate
system by
\begin{equation}\label{huohuioy89gjjhjffffff3478zzrrZZZhjhhjhhjjhhffGGhjjhuiiuihhj}
\begin{cases}
\mathcal{T}^{00}=0
\\
\mathcal{T}^{0j}=\mathcal{T}^{j0}=0\quad\quad\forall\, j=1,2,3
\\
\mathcal{T}^{ij}:=T_{ij}\quad\quad\forall\, i,j=1,2,3\,,
\end{cases}
\end{equation}
Then by inserting \er{noninchgravortbstrjgghguittu2intrrrZZygjyghhj}
and \er{uguyytfddddgghjjghjjjihohjjk} into
\er{fgjfjhgghhgjghjhjkkkkgggh}, we can prove that the field
$\{\mathcal{T}^{ij}\}_{0\leq i,j\leq 3}$ defined by
\er{huohuioy89gjjhjffffff3478zzrrZZZhjhhjhhjjhhffGGhjjhuiiuihhj} is
a two times \underline{contravariant} tensor field on the group
$\mathcal{S}_0$. Indeed, by \er{fgjfjhgghhgjghjhjkkkkgggh} for every
two times contravariant tensor field $\{a^{ij}\}_{0\leq i,j\leq 3}$
we have
\begin{multline}\label{fgjfjhgghhgjghjhjkkkkggghkjjk}
a'^{mn}=\frac{\partial f^{(m)}}{\partial x^0}\frac{\partial
f^{(n)}}{\partial x^0}a^{00}+\sum_{k=1}^{3}\frac{\partial
f^{(m)}}{\partial x^k}\frac{\partial f^{(n)}}{\partial
x^0}a^{k0}+\sum_{j=1}^{3}\frac{\partial f^{(m)}}{\partial
x^0}\frac{\partial f^{(n)}}{\partial
x^j}a^{0j}\\+\sum_{j=1}^{3}\sum_{k=1}^{3}\frac{\partial
f^{(m)}}{\partial x^k}\frac{\partial f^{(n)}}{\partial
x^j}a^{kj}\quad\quad\forall\, m,n=0,1,2,3.
\end{multline}
Then, since by \er{noninchgravortbstrjgghguittu2intrrrZZygjyghhj} we
have $\frac{\partial f^{(0)}}{\partial x^0}=1$, $\frac{\partial
f^{(0)}}{\partial x^k}=0$ $\forall k=1,2,3$ and $\frac{\partial
f^{(m)}}{\partial x^k}=A_{mk}\left(\frac{x^0}{c}\right)$ $\forall
k,m=1,2,3$, in the case where $a^{00}=0$ and $a^{0j}=a^{j0}=0$
$\forall j=1,2,3$  we rewrite \er{fgjfjhgghhgjghjhjkkkkggghkjjk} as:
\begin{equation}\label{fgjfjhgghhgjghjhjkkkkggghjjkjhgj}
\begin{cases}
a'^{00}=0\\
a'^{j0}=a'^{0j}=0\quad\quad\forall\, j=1,2,3,
\\
a'^{mn}=\sum_{j=1}^{3}\sum_{k=1}^{3}A_{mk}\left(\frac{x^0}{c}\right)A_{nj}\left(\frac{x^0}{c}\right)a^{kj}\quad\quad\forall\,
m,n=1,2,3.
\end{cases}
\end{equation}
that is compatible with
\er{huohuioy89gjjhjffffff3478zzrrZZZhjhhjhhjjhhffGGhjjhuiiuihhj} and
\er{uguyytfddddgghjjghjjjihohjjk}.

In particular, if we consider the $9$-component matrix field $I$
that defined in every cartesian coordinate system as
$I:=\left\{\delta_{ij}\right\}_{1,j=1,2,3}\in \R^{3\times 3}$, where
\begin{equation}\label{huohuioy89gjjhjffffff3478zzrrZZZhjhhjhhjjhhffGGhjjhuiiuihhjkklpkl}
\delta_{ij}:=
\begin{cases}
1\quad\text{if}\quad i=j
\\
0\quad\text{if}\quad i\neq j,
\end{cases}
\end{equation}
which is a proper matrix field, since
\begin{equation}\label{uguyytfddddgghjjghjjjihohjjkkkook}
I=A(t)\cdot I\cdot \left\{A(t)\right\}^{-1},
\end{equation}
then the $16$-component field $\{{\Theta}^{ij}\}_{0\leq i,j\leq 3}$
defined in every non-inertial cartesian coordinate system by
\begin{equation}\label{huohuioy89gjjhjffffff3478zzrrZZZhjhhjhhjjhhffGGhjjhuiiuihhjkkoioi}
\begin{cases}
{\Theta}^{00}=0
\\
{\Theta}^{0j}={\Theta}^{j0}=0\quad\quad\forall\, j=1,2,3
\\
{\Theta}^{ij}:=\delta_{ij}\quad\quad\forall\, i,j=1,2,3
\end{cases}
\end{equation}
is a two times contravariant tensor field on the group
$\mathcal{S}_0$ and moreover, this tensor is symmetric. We call
$\{{\Theta}^{ij}\}_{0\leq i,j\leq 3}$ the contravariant tensor of
the three-dimensional geometry.

Next, the scalar field $\tau:=\tau(x^0,x^1,x^2,x^3)$, defined in
every cartesian coordinate system as
\begin{equation}\label{uguyytfddddgghkjhhjuuiui}
\tau:=\frac{x^0}{c}=t,
\end{equation}
is a scalar on the group $\mathcal{S}_0$. Here $t$ is the global
non-relativistic time. Moreover, by
\er{fgjfjhgghhgjghjhjkkkkgjghghuiiiulkkjlkklplikkl} and
\er{uguyytfddddgghkjhhjuuiui}, the four-component field
$(v_0,v_1,v_2,v_3)$ defined by:
\begin{equation}\label{fgjfjhgghhgjghjhjkkkkgjghghuiiiulkkjlkklplikklkjljghhgghk}
v_0:=c\,\frac{\partial \tau}{\partial
x^0}(x^0,x^1,x^2,x^3)=1\quad\text{and}\quad v_j:=c\,\frac{\partial
\tau}{\partial x^j}(x^0,x^1,x^2,x^3)=0\quad\quad\forall\,j=1,2,3,
\end{equation}
is a \underline{four-covector} field on the group $\mathcal{S}_0$.

Finally, consider a motion of a classical particle with inertial
mass $m$, charge $\sigma$, place $\vec r(t)$ and velocity $\vec
u(t)=\vec r'(t)$ in the outer gravitational field with the vectorial
gravitational potential $\vec v(\vec x,t)$, the outer
electromagnetic field with vectorial and scalar potentials $\vec
A(\vec x,t)$ and $\vec \Psi(\vec x,t)$, and additional conservative
field with scalar potential $V(\vec x,t)$ ruled by a Lagrangian
\er{vhfffngghkjgghfjjint}:
\begin{equation}\label{vhfffngghkjgghfjjuuu}
L_0\left(\frac{d\vec r}{dt},\vec
r,t\right):=\frac{m}{2}\left|\frac{d\vec r}{dt}-\vec v(\vec
r,t)\right|^2-\sigma\left(\Psi(\vec r,t)-\frac{1}{c}\vec A(\vec
r,t)\cdot\frac{d\vec r}{dt}\right)+V(\vec r,t).
\end{equation}
Then $L_0$ is a scalar on the group $\mathcal{S}_0$. Moreover,
consider the generalized momentum of the particle $m$ by
\er{guytyurtydftyiujhint}:
\begin{equation}\label{guytyurtydftyiujhuuu}
\vec P:=\nabla_{\vec r'}L_0\left(\vec r',\vec r,t\right)=m
\frac{d\vec r}{dt}-m\vec v(\vec r,t)+\frac{\sigma}{c}\vec A(\vec
r,t),
\end{equation}
consider a Hamiltonian
\begin{equation}\label{vhfffngghkjgghfjjhyjjfguuu}
H_0\left(\vec P,\vec r,t\right):=\vec P\cdot\frac{d\vec
r}{dt}-L_0\left(\frac{d\vec r}{dt},\vec r,t\right),
\end{equation}
which by \er{vhfffngghkjgghfjjghghghint} satisfies
\begin{equation}\label{vhfffngghkjgghfjjghghghuuu} H_0\left(\vec
P,\vec r,t\right)= \vec P\cdot\vec v(\vec r,t)+
\frac{1}{2m}\left|\vec P-\frac{\sigma}{c}\vec A(\vec
r,t)\right|^2+\sigma\left(\Psi(\vec r,t)-\frac{1}{c}\vec A(\vec
r,t)\cdot\vec v(\vec r,t)\right)-V\left(\vec r,t\right),
\end{equation}
and furthermore, define the four-dimensional generalized momentum
$(P_0,P_1,P_2,P_3)$ as:
\begin{equation}\label{fgjfjhgghhgjghjhjijhojihjhjjijhjjjljljpkhkhjgjg}
(P_0,P_1,P_2,P_3):=\left(\frac{1}{c}H_0,-\vec
P\right)\quad\text{where}\quad
P_0=\frac{1}{c}H_0\;\;\text{and}\;\;(P_1,P_2,P_3)=-\vec P,
\end{equation}
Then, since by \er{vhfffngghkjgghfjjghghghuuu} and
\er{guytyurtydftyiujhuuu}, under the change of non-inertial
cartesian coordinate system $H_0$ and $\vec P$ transform as
\begin{equation}\label{vhfffngghhjghhgPPNghghghutghfflklhjkjhjhjjgjkghhjhhjhjkhjghlkkll}
\begin{cases}
H_0'=
H_0+\left(\frac{dA}{dt}(t)\cdot\vec x+\frac{d\vec
z}{dt}(t)\right)\cdot\left(A(t)\cdot\vec P\right)
\\
\vec P'=A(t)\cdot \vec P,
\end{cases}
\end{equation}
by comparing
\er{vhfffngghhjghhgPPNghghghutghfflklhjkjhjhjjgjkghhjhhjhjkhjghlkkll}
with \er{fgjfjhgghhgjghjhjijhojpiiihjhjuiouj} we deduce that the
four-dimensional momentum $(P_0,P_1,P_2,P_3)$ is a
four-\underline{covector} on the group $\mathcal{S}_0$.

\subsection{Pseudo-metric tensors of the four-dimensional
space-time}\label{jjjjjkkk}
Consider $\{g^{ij}\}_{0\leq i,j\leq 3}$
to be a two times contravariant tensor field on the group
$\mathcal{S}_0$, defined by
\begin{equation}\label{hoyuiouigyfghgjh3478zzrrZZffGGjkkjojj}
g^{ij}:=v^iv^j-{\Theta}^{ij}\quad\quad\forall\,i,j=0,1,2,3,
\end{equation}
where $\{{\Theta}^{ij}\}_{0\leq i,j\leq 3}$ is the contravariant
tensor of the three-dimensional geometry, defined by
\er{huohuioy89gjjhjffffff3478zzrrZZZhjhhjhhjjhhffGGhjjhuiiuihhjkkoioi}
and being a two times contravariant tensor, and $(v^0,v^1,v^2,v^3)$
is the four-dimensional gravitational potential, defined by
\er{fgjfjhgghhgjghjhjijhojihjhjjijhjjjjjuiijjjk} and being a
four-vector. Then by \er{fgjfjhgghhgjghjhjkkkkgjghghjljl} we obtain
that $\{g^{ij}\}_{0\leq i,j\leq 3}$ is indeed a two times
contravariant tensor field on the group $\mathcal{S}_0$ and
moreover, this tensor is symmetric. Moreover, by
\er{huohuioy89gjjhjffffff3478zzrrZZZhjhhjhhjjhhffGGhjjhuiiuihhjkkoioi}
and \er{fgjfjhgghhgjghjhjijhojihjhjjijhjjjjjuiijjjk} we have:
\begin{equation}\label{hoyuiouigyfghgjh3478zzrrZZffGGjkkj}
\begin{cases}
g^{00}=1\\
g^{ij}=-\delta_{ij}+\frac{v^iv^j}{c^2}\quad\forall 1\leq i,j\leq 3\\
g^{0j}=g^{j0}=\frac{v^j}{c}\quad\forall 1\leq j\leq 3,
\end{cases}
\end{equation}
where $\vec v=(v^1,v^2,v^3)$ is the three-dimensional vectorial
gravitational potential. We call the tensor $\{g^{ij}\}_{0\leq
i,j\leq 3}$ the contravariant pseudo-metric tensor of the
four-dimensional space-time. Next consider a $16$-component field
$\{g_{ij}\}_{0\leq i,j\leq 3}$ defined by
\begin{equation}\label{hoyuiouigyfg3478zzrrZZffGGhhjhj}
\begin{cases}
g_{00}=1-\frac{|\vec v|^2}{c^2}\\
g_{ij}=-\delta_{ij}\quad\forall 1\leq i,j\leq 3\\
g_{0j}=g_{j0}=\frac{v^j}{c}\quad\forall 1\leq j\leq 3.
\end{cases}
\end{equation}
 Then
\begin{equation*}
\sum_{k=0}^{3}g_{0k}g^{k0}=g_{00}g^{00}+\sum_{k=1}^{3}g_{0k}g^{k0}=
1-\frac{|\vec v|^2}{c^2}+\frac{|\vec v|^2}{c^2}=1,
\end{equation*}
\begin{equation*}
\sum_{k=0}^{3}g_{ik}g^{kj}=g_{i0}g^{0j}+\sum_{k=1}^{3}g_{ik}g^{kj}=
\frac{v^iv^j}{c^2}+\delta_{ij}-\frac{v^iv^j}{c^2}=\delta_{ij}\quad\forall
1\leq i,j\leq 3,
\end{equation*}
and
\begin{equation*}
\sum_{k=0}^{3}g_{ik}g^{k0}=g_{i0}g^{00}+\sum_{k=1}^{3}g_{ik}g^{k0}=\frac{v^i}{c}-\frac{v^i}{c}=0
\quad\forall 1\leq i\leq 3,
\end{equation*}
\begin{multline*}
\sum_{k=0}^{3}g_{0k}g^{kj}=g_{00}g^{0j}+\sum_{k=1}^{3}g_{0k}g^{kj}=
\left(1-\frac{|\vec v|^2}{c^2}\right)\frac{v^j}{c}
-\sum_{k=1}^{3}\frac{v^k}{c}\left(\delta_{kj}-\frac{v^kv^j}{c^2}\right)=0
\quad\forall 1\leq j\leq 3,
\end{multline*}
where $\{g^{ij}\}_{0\leq i,j\leq 3}$ is the contravariant
pseudo-metric tensor of the four-dimensional space-time, defined by
\er{hoyuiouigyfghgjh3478zzrrZZffGGjkkj}. So,
\begin{equation}\label{fgjfjhgghhgjghjhjkkkkgjghghuiiiuujhjh}
\sum_{k=0}^{3}g^{ik}g_{kj}=\begin{cases}
1\quad\text{if}\quad i=j\\
0\quad\text{if}\quad i\neq j
\end{cases}\quad\quad\forall\, i,j=0,1,2,3.
\end{equation}
Therefore, by comparing \er{fgjfjhgghhgjghjhjkkkkgjghghuiiiuujhjh}
and \er{fgjfjhgghhgjghjhjkkkkgjghghuiiiu} we deduce that
$\{g_{ij}\}_{i,j=0,1,2,3}$ is a two times covariant tensor on the
group $\mathcal{S}_0$, and moreover, this tensor is symmetric. We
call the tensor $\{g_{ij}\}_{0\leq i,j\leq 3}$ covariant
pseudo-metric tensor of the four-dimensional space-time. Using
\er{fgjfjhgghhgjghjhjkkkkgjghghuiiiuujhjh} we also obtain that the
pseudo-metric tensors $\{g_{ij}\}_{i,j=0,1,2,3}$ and
$\{g^{ij}\}_{0\leq i,j\leq 3}$ are non-degenerate and they are
inverse of each other. Moreover, it can be easily calculated that if
we consider the $4\times 4$-matrix:
\begin{equation}\label{fgjfjhgghhgjghjhjkkkkgjghghuiiiuujhjhjklj}
G=\{g_{ij}\}_{0\leq i,j\leq 3},
\end{equation}
then
\begin{equation}\label{fgjfjhgghhgjghjhjkkkkgjghghuiiiuujhjh1}
\text{det}\,G=-1.
\end{equation}

Thus, with the covariant and contravariant pseudo-metric tensors we
can lower and lift indexes of arbitrary tensors. In particular given
a four-covector $(a_0,a_1,a_2,a_3)$ and a four-vector
$(b^0,b^1,b^2,b^3)$ on the group $\mathcal{S}_0$ we can define the
corresponding lifted four-vector $(a^0,a^1,a^2,a^3)$ and the
corresponded lowered four-covector $(b_0,b_1,b_2,b_3)$ by
\begin{equation}\label{fgjfjhgghhgjghjhjkkkkgjghghuiiiulkkjKK}
(a^0,a^1,a^2,a^3):=\Big\{\sum_{k=0}^{3}g^{mk}a_{k}\Big\}_{m=0,1,2,3}\quad\text{and}\quad
(b_0,b_1,b_2,b_3):=\Big\{\sum_{k=0}^{3}g_{mk}b^{k}\Big\}_{m=0,1,2,3}
\end{equation}
Then by \er{hoyuiouigyfghgjh3478zzrrZZffGGjkkj},
\er{hoyuiouigyfg3478zzrrZZffGGhhjhj} and
\er{fgjfjhgghhgjghjhjkkkkgjghghuiiiulkkjKK} we have:
\begin{equation}\label{fgjfjhgghhgjghjhjkkkkgjghghuiiiulkkjKKyuyyu}
a^0=a_{0}+\sum_{k=1}^{3}\frac{1}{c}v^k a_{k}\quad\text{and}\quad
a^m=-a_{m}+\frac{1}{c}\bigg(a_{0}+\sum_{k=1}^{3}\frac{1}{c}v^k
a_{k}\bigg)v^m\quad\quad\forall m=1,2,3,
\end{equation}
and
\begin{multline}\label{fgjfjhgghhgjghjhjkkkkgjghghuiiiulkkjKKyuyyuuuui}
b_0=\left(1-\frac{|\vec
v|^2}{c^2}\right)b^{0}+\sum_{k=1}^{3}\frac{1}{c}v^k
b^{k}=b^0-\sum_{k=1}^{3}\frac{1}{c}v^k\left(-
b^{k}+\frac{1}{c}b^0v^k\right)\\ \quad\text{and}\quad
b_m=-b^{m}+\frac{1}{c} b^{0}v^m\quad\quad\forall m=1,2,3.
\end{multline}
We also can rewrite \er{fgjfjhgghhgjghjhjkkkkgjghghuiiiulkkjKKyuyyu}
and \er{fgjfjhgghhgjghjhjkkkkgjghghuiiiulkkjKKyuyyuuuui} as:
\begin{equation}\label{fgjfjhgghhgjghjhjkkkkgjghghuiiiulkkjKKyuyyuioui}
a^0=a_{0}+\sum_{k=1}^{3}\frac{1}{c}v^k a_{k}\quad\text{and}\quad
a^m=-a_{m}+\frac{1}{c}a^0v^m\quad\quad\forall m=1,2,3,
\end{equation}
and
\begin{equation}\label{fgjfjhgghhgjghjhjkkkkgjghghuiiiulkkjKKyuyyuuuuiiuii}
b_m=-b^{m}+\frac{1}{c} b^{0}v^m \quad\text{and}\quad
b_0=b^0-\sum_{k=1}^{3}\frac{1}{c}v^kb_k \quad\quad\forall m=1,2,3.
\end{equation}
In particular, if we consider the scalar field $\Lambda$ on the
group $\mathcal{S}_0$ defined by:
\begin{equation}\label{fgjfjhgghhgjghjhjkkkkgjghghuiiiulkkjKKyuyyu0io}
\Lambda:=b^0a_0+\sum_{k=1}^{3}b^ka_k
\end{equation}
then by inserting
\er{fgjfjhgghhgjghjhjkkkkgjghghuiiiulkkjKKyuyyuioui} and
\er{fgjfjhgghhgjghjhjkkkkgjghghuiiiulkkjKKyuyyuuuuiiuii} into
\er{fgjfjhgghhgjghjhjkkkkgjghghuiiiulkkjKKyuyyu0io} we deduce:
\begin{equation}\label{fgjfjhgghhgjghjhjkkkkgjghghuiiiulkkjKKyuyyu0ioioio}
\Lambda=b^0\left(a^0-\sum_{k=1}^{3}\frac{1}{c}v^ka_k\right)+\sum_{k=1}^{3}\left(-b_{k}+\frac{1}{c}b^0v^k\right)a_k=b^0a^0-\sum_{k=1}^{3}b_{k}a_k.
\end{equation}
So,
\begin{equation}\label{fgjfjhgghhgjghjhjkkkkgjghghuiiiulkkjKKyuyyu0ioioiogghghgh}
\Lambda=b^0a_0+\sum_{k=1}^{3}b^ka_k=b^0a^0-\sum_{k=1}^{3}b_{k}a_k.
\end{equation}

Next, if for every speed-like vector field $\vec u$ we consider the
four-vector field $(u^0,u^1,u^2,u^3)$ defined by
\er{fgjfjhgghhgjghjhjijhojihjhjjijhjjjjjuii} as:
\begin{equation}\label{fgjfjhgghhgjghjhjijhojihjhjjijhjjjjjuiikkj}
(u^0,u^1,u^2,u^3):=\left(1,\frac{1}{c}\vec
u\right)\quad\text{where}\quad
u^0=1\;\;\text{and}\;\;(u^1,u^2,u^3)=\frac{1}{c}\vec
u\in\mathbb{R}^3,
\end{equation}
then, by \er{fgjfjhgghhgjghjhjkkkkgjghghuiiiulkkjKKyuyyuuuuiiuii}
the corresponding lowered four-covector field $(u_0,u_1,u_2,u_3)$
satisfies:
\begin{multline}\label{fgjfjhgghhgjghjhjijhojihjhjjijhjjjjjuiikkjjnj}
(u_0,u_1,u_2,u_3):=\left(1+\frac{1}{c^2}\left(\vec u-\vec
v\right)\cdot\vec v,-\frac{1}{c}\left(\vec u-\vec
v\right)\right)\quad\text{where}\quad\\
u_0=1+\frac{1}{c^2}\left(\vec u-\vec v\right)\cdot\vec
v\;\;\text{and}\;\;(u_1,u_2,u_3)=-\frac{1}{c}\left(\vec u-\vec
v\right)\in\mathbb{R}^3.
\end{multline}
Moreover, in the case where $(u^0,u^1,u^2,u^3)$  is a
four-dimensional speed, we call the corresponding lowered
four-covector field $(u_0,u_1,u_2,u_3)$ by the name four-dimensional
cospeed. In particular, if we consider the four-dimensional
gravitational potential $(v^0,v^1,v^2,v^3)$ defined by
\er{fgjfjhgghhgjghjhjijhojihjhjjijhjjjjjuiijjjk}:
\begin{equation}\label{fgjfjhgghhgjghjhjijhojihjhjjijhjjjjjuiijjjkhjhjhj}
(v^0,v^1,v^2,v^3):=\left(1,\frac{1}{c}\vec
v\right)\quad\text{where}\quad
v^0=1\;\;\text{and}\;\;(v^1,v^2,v^3)=\frac{1}{c}\vec v,
\end{equation}
then by \er{fgjfjhgghhgjghjhjijhojihjhjjijhjjjjjuiikkjjnj} we obtain
that the corresponding lowered four-covector field
$(v_0,v_1,v_2,v_3)$, that we call the four-covector of gravitational
potential, satisfies:
\begin{equation}\label{fgjfjhgghhgjghjhjijhojihjhjjijhjjjjjuiikkjjnj1}
(v_0,v_1,v_2,v_3):=\left(1,0,0,0\right)\quad\text{where}\quad
v_0=1\;\;\text{and}\;\;(v_1,v_2,v_3)=0:=(0,0,0).
\end{equation}
Note that the four-covector of gravitational potential, defined by
\er{fgjfjhgghhgjghjhjijhojihjhjjijhjjjjjuiikkjjnj1} coincides with
the four-covector defined by
\er{fgjfjhgghhgjghjhjkkkkgjghghuiiiulkkjlkklplikklkjljghhgghk} as
the gradient of the scalar of global time. Next, by
\er{fgjfjhgghhgjghjhjijhojihjhjjijhjjjjjuiijjjkhjhjhj} and
\er{fgjfjhgghhgjghjhjijhojihjhjjijhjjjjjuiikkjjnj1} we clearly have:
\begin{equation}\label{fgjfjhgghhgjghjhjijhojihjhjjijhjjjjjuiijjjkhjhjhjuiiuuuyu}
c^2\left(\sum_{j=0}^{3}\sum_{k=0}^{3}\,g^{jk}\frac{\partial\tau}{\partial
x^j}\frac{\partial\tau}{\partial
x^k}\right)=\sum_{j=0}^{3}\sum_{k=0}^{3}g^{jk}v_jv_k=\sum_{j=0}^{3}\sum_{k=0}^{3}g_{jk}v^jv^k=\sum_{j=0}^{3}v^jv_j=1,
\end{equation}
where $\tau$ is the scalar of the global time on the group
$\mathcal{S}_0$, defined by \er{uguyytfddddgghkjhhjuuiui}. Finally,
we clearly have
\begin{equation}\label{fgjfjhfgfgfgfh}
\sum_{k=0}^{3}{\Theta}^{mk}\frac{\partial\tau}{\partial
x^k}=\sum_{k=0}^{3}{\Theta}^{mk}v_k=0\quad\quad\forall\, m=0,1,2,3,
\end{equation}
where $\Theta^{ij}$ is the contravariant tensor of the
three-dimensional geometry, defined by
\er{huohuioy89gjjhjffffff3478zzrrZZZhjhhjhhjjhhffGGhjjhuiiuihhjkkoioi}.

More generally, if for every speed-like vector field $\vec u$ and
every proper scalar field $\sigma$  we consider the four-vector
field $(b^0,b^1,b^2,b^3)$ on the group $\mathcal{S}_0$ defined by
\er{fgjfjhgghhgjghjhjijhojihjhjjijhjjjjjuiijjjkl} as:
\begin{equation}\label{fgjfjhgghhgjghjhjijhojihjhjjijhjjjjjuiijjjklhhh}
(b^0,b^1,b^2,b^3):=\left(\sigma,\frac{\sigma}{c}\vec
u\right)\quad\text{where}\quad
b^0=\sigma\;\;\text{and}\;\;(b^1,b^2,b^3)=\frac{\sigma}{c}\vec u,
\end{equation}
then by \er{fgjfjhgghhgjghjhjkkkkgjghghuiiiulkkjKKyuyyuuuuiiuii} the
corresponding lowered four-covector field $(b_0,b_1,b_2,b_3)$
satisfies:
\begin{multline}\label{fgjfjhgghhgjghjhjijhojihjhjjijhjjjjjuiikkjjnjjbj}
(b_0,b_1,b_2,b_3):=\left(\sigma\left(1+\frac{1}{c^2}\left(\vec
u-\vec v\right)\cdot\vec v\right),-\frac{\sigma}{c}\left(\vec u-\vec
v\right)\right)\quad\text{where}\quad\\
b_0=\sigma\left(1+\frac{1}{c^2}\left(\vec u-\vec v\right)\cdot\vec
v\right)\;\;\text{and}\;\;(b_1,b_2,b_3)=-\frac{\sigma}{c}\left(\vec
u-\vec v\right).
\end{multline}
In particular, if we consider the field of four-vector of the moment
of a particle $(p^0,p^1,p^2,p^3)$ defined by
\er{fgjfjhgghhgjghjhjijhojihjhjjijhjjjjjuiijjjklihh} as
\begin{equation}\label{fgjfjhgghhgjghjhjijhojihjhjjijhjjjjjuiijjjklihhjjj}
(p^0,p^1,p^2,p^3):=\left(m,\frac{1}{c}(m\vec
u)\right)\quad\text{where}\quad
p^0=m\;\;\text{and}\;\;(p^1,p^2,p^3)=\frac{1}{c}(m\vec u),
\end{equation}
where $m$ is the mass of the particle and $\vec u$ is the velocity
of the particle, then the corresponding lowered four-covector field
$(p_0,p_1,p_2,p_3)$, which we call the four-covector of momentum,
satisfies:
\begin{multline}\label{fgjfjhgghhgjghjhjijhojihjhjjijhjjjjjuiikkjjnjjbjhjh}
(p_0,p_1,p_2,p_3):=\left(m\left(1+\frac{1}{c^2}\left(\vec u-\vec
v\right)\cdot\vec v\right),-\frac{m}{c}\left(\vec u-\vec
v\right)\right)\quad\text{where}\quad\\
p_0=m\left(1+\frac{1}{c^2}\left(\vec u-\vec v\right)\cdot\vec
v\right)\;\;\text{and}\;\;(p_1,p_2,p_3)=-\frac{m}{c}\left(\vec
u-\vec v\right).
\end{multline}
In particular, the scalar field $J_0$ defined by
\begin{equation}\label{fgjfjhgghhgjghjhjkkkkgjghghuiiiulkkjKKyuyyu0ioioiogghghghgghghhg}
J_0:=-\frac{c^2}{2m}\left(p^0p_0+\sum_{k=1}^{3}p^kp_k\right),
\end{equation}
by \er{fgjfjhgghhgjghjhjkkkkgjghghuiiiulkkjKKyuyyu0ioioiogghghgh},
\er{fgjfjhgghhgjghjhjijhojihjhjjijhjjjjjuiijjjklihhjjj} and
\er{fgjfjhgghhgjghjhjijhojihjhjjijhjjjjjuiikkjjnjjbjhjh} satisfies:
\begin{equation}\label{fgjfjhgghhgjghjhjkkkkgjghghuiiiulkkjKKyuyyu0ioioiogghghghgghgh}
J_0=\frac{mc^2}{2}\left(\frac{1}{c^2}\left|\vec u-\vec
v\right|^2-1\right)=\frac{m}{2}\left|\vec u-\vec
v\right|^2-\frac{mc^2}{2}.
\end{equation}
Moreover, if we consider the  four-dimensional electric current
$(j^0,j^1,j^2,j^3)$ defined by
\er{fgjfjhgghhgjghjhjijhojihjhjjijhjjjjjuiijjjklihhojjjo} as
\begin{equation}\label{fgjfjhgghhgjghjhjijhojihjhjjijhjjjjjuiijjjklihhojjjokjkj}
(j^0,j^1,j^2,j^3):=\left(\rho,\frac{1}{c}\,\vec
j\right)\quad\text{where}\quad
j^0=\rho\;\;\text{and}\;\;(j^1,j^2,j^3)=\frac{1}{c}\,\vec j,
\end{equation}
where $\rho$ is the electric charge density and $\vec j$ is the
electric current density, then the corresponding lowered
four-covector field $(j_0,j_1,j_2,j_3)$, which we call the
four-covector of current, satisfies:
\begin{multline}\label{fgjfjhgghhgjghjhjijhojihjhjjijhjjjjjuiikkjjnjjbjhjhiyuyug}
(j_0,j_1,j_2,j_3):=\left(\rho+\frac{1}{c^2}\left(\vec j-\rho\vec
v\right)\cdot\vec v,-\frac{1}{c}\left(\vec j-\rho\vec
v\right)\right)\quad\text{where}\quad\\
j_0=\rho+\frac{1}{c^2}\left(\vec j-\rho\vec v\right)\cdot\vec
v\;\;\text{and}\;\;(j_1,j_2,j_3)=-\frac{1}{c}\left(\vec j-\rho\vec
v\right).
\end{multline}
Finally, if $\Psi$ is the scalar electromagnetic potential and $\vec
A$ is the vector electromagnetic potential and we consider the
four-\underline{covector} field of four dimensional electromagnetic
potential $(A_0,A_1,A_2,A_3)$, defined by
\er{fgjfjhgghhgjghjhjijhojihjhjjijhjjjljljpk} as:
\begin{equation}\label{fgjfjhgghhgjghjhjijhojihjhjjijhjjjljljpkikhj}
(A_0,A_1,A_2,A_3)=(\Psi,-\vec A)\quad\text{where}\quad
A_0=\Psi\;\;\text{and}\;\;(A_1,A_2,A_3)=-\vec A,
\end{equation}
then by inserting \er{fgjfjhgghhgjghjhjijhojihjhjjijhjjjljljpkikhj}
into \er{fgjfjhgghhgjghjhjkkkkgjghghuiiiulkkjKKyuyyuioui} we deduce
that the corresponding lifted four-vector field $(A^0,A^1,A^2,A^3)$,
which we call the four-vector of electromagnetic potential,
satisfies:
\begin{equation}\label{fgjfjhgghhgjghjhjkkkkgjghghuiiiulkkjKKyuyyuiouigjguuuiui}
A^0=\Psi-\frac{1}{c}\vec v\cdot\vec
A\quad\text{and}\quad(A^1,A^2,A^3)=\vec
A+\frac{1}{c}\left(\Psi-\frac{1}{c}\vec v\cdot\vec A\right)\vec v.
\end{equation}
On the other hand, the proper scalar electromagnetic potential
$\Psi_0$ was defined by
\er{vhfffngghhjghhgPPNghghghutghffugghjhjkjjkl} as:
\begin{equation}\label{vhfffngghhjghhgPPNghghghutghffugghjhjkjjklhjgh}
\Psi_0:=\Psi-\frac{1}{c}\vec A\cdot\vec v.
\end{equation}
Thus we rewrite
\er{fgjfjhgghhgjghjhjkkkkgjghghuiiiulkkjKKyuyyuiouigjguuuiui} as:
\begin{equation}\label{fgjfjhgghhgjghjhjkkkkgjghghuiiiulkkjKKyuyyuiouigjguuuiui1}
A^0=\Psi_0\quad\text{and}\quad(A^1,A^2,A^3)=\vec
A+\frac{1}{c}\Psi_0\vec v.
\end{equation}

Next given a two times covariant tensor $\{c_{mn}\}_{m,n=0,1,2,3}$
on the group $\mathcal{S}_0$ by
\er{fgjfjhgghhgjghjhjkkkkgjghghuiiiulkkjlkkl} we consider two times
contravariant tensor
on $\mathcal{S}_0$:
$\{c^{mn}\}_{m,n=0,1,2,3}$
defined
by:
\begin{equation}\label{fgjfjhgghhgjghjhjkkkkgjghghuiiiulkkjlkklKK}
c^{mn}:=\sum_{k=0}^{3}\sum_{j=0}^{3}g^{mj}g^{nk}c_{jk}
\quad\quad\forall\,
m,n=0,1,2,3.
\end{equation}
We rewrite \er{fgjfjhgghhgjghjhjkkkkgjghghuiiiulkkjlkklKK} as:
\begin{multline}\label{khjhhkfgjfjhgghhgjghjhjkkkkgjghghuiiiulkkjlkklKKgfg}
c^{mn}=g^{m0}g^{n0}c_{00}+\sum_{k=1}^{3}g^{m0}g^{nk}c_{0k}+\sum_{j=1}^{3}g^{mj}g^{n0}c_{j0}+\sum_{k=1}^{3}\sum_{j=1}^{3}g^{mj}g^{nk}c_{jk}
\quad\forall\,
m,n=0,1,2,3.
\end{multline}
In particular, by inserting \er{hoyuiouigyfghgjh3478zzrrZZffGGjkkj}
and \er{hoyuiouigyfg3478zzrrZZffGGhhjhj} into
\er{khjhhkfgjfjhgghhgjghjhjkkkkgjghghuiiiulkkjlkklKKgfg} we deduce:
\begin{equation}\label{khjhhkfgjfjhgghhgjghjhjkkkkgjghghuiiiulkkjlkklKKgfgjhjjghgjh}
\begin{cases}
c^{00}=c_{00}+\sum_{k=1}^{3}\frac{v^k}{c}c_{0k}+\sum_{j=1}^{3}\frac{v^j}{c}c_{j0}+\sum_{k=1}^{3}\sum_{j=1}^{3}\frac{v^j}{c}\frac{v^k}{c}c_{jk}
\\
c^{m0}=
\frac{v^m}{c}c^{00}-c_{m0}-\sum_{k=1}^{3}\frac{v^k}{c}c_{mk}
\quad\forall\, m=1,2,3,
\\
c^{0n}=\frac{v^n}{c}c^{00}
-c_{0n} -\sum_{j=1}^{3}\frac{v^j}{c}c_{jn} \quad\forall\, n=1,2,3,
\\
c^{mn}=\frac{v^m}{c}\frac{v^n}{c}c^{00}
-\sum_{k=1}^{3}\frac{v^n}{c}\frac{v^k}{c}c_{mk}-\sum_{j=1}^{3}\frac{v_m}{c}\frac{v^j}{c}c_{jn}
-\frac{v^m}{c}c_{0n}-\frac{v^n}{c}c_{m0}+c_{mn}\quad\forall\,
m,n=1,2,3.
\end{cases}
\end{equation}
We rewrite
\er{khjhhkfgjfjhgghhgjghjhjkkkkgjghghuiiiulkkjlkklKKgfgjhjjghgjh}
as:
\begin{equation}\label{khjhhkfgjfjhgghhgjghjhjkkkkgjghghuiiiulkkjlkklKKgfgjhjjghgjh1}
\begin{cases}
c^{00}=c_{00}+\sum_{k=1}^{3}\frac{v^k}{c}c_{0k}+\sum_{j=1}^{3}\frac{v^j}{c}c_{j0}+\sum_{k=1}^{3}\sum_{j=1}^{3}\frac{v^j}{c}\frac{v^k}{c}c_{jk}
\\
c^{m0}=
\frac{v^m}{c}c^{00}-c_{m0}-\sum_{k=1}^{3}\frac{v^k}{c}c_{mk}
\quad\forall\, m=1,2,3,
\\
c^{0n}=\frac{v^n}{c}c^{00}
-c_{0n} -\sum_{j=1}^{3}\frac{v^j}{c}c_{jn} \quad\forall\, n=1,2,3,
\\
c^{mn}=\frac{v^m}{c}c^{0n}+\frac{v^n}{c}c^{m0}-\frac{v^m}{c}\frac{v^n}{c}c^{00}
+c_{mn}\quad\forall\,
m,n=1,2,3.
\end{cases}
\end{equation}
In particular if the tensor $\{c_{mn}\}_{m,n=0,1,2,3}$ is
antisymmetric, i.e. $c_{mn}=-c_{nm}\;\forall \,m,n=0,1,2,3$, then we
simplify
\er{khjhhkfgjfjhgghhgjghjhjkkkkgjghghuiiiulkkjlkklKKgfgjhjjghgjh1}
as
\begin{equation}\label{khjhhkfgjfjhgghhgjghjhjkkkkgjghghuiiiulkkjlkklKKgfgjhjjghgjhhjhj}
\begin{cases}
c^{00}=0
\\
c^{mm}=0 \quad\forall\, m=1,2,3,
\\
c^{0m}=-c^{m0}=
-c_{0m}+\sum_{k=1}^{3}\frac{v^k}{c}c_{mk} \quad\forall\, m=1,2,3,
\\
c^{mn}=\frac{v^m}{c}c^{0n}-\frac{v^n}{c}c^{0m}
+c_{mn}\quad\forall\, m,n=1,2,3.
\end{cases}
\end{equation}
In particular, if $\{F_{ij}\}_{0\leq i,j\leq 3}$ is the
antisymmetric two times covariant tensor field of the
electromagnetic field on the group $\mathcal{S}_0$, which by
\er{huohuioy89gjjhjffffff3478zzrrZZZhjhhjhhjjhhffGGhjjhiuuijkjjk}
satisfies:
\begin{equation}\label{huohuioy89gjjhjffffff3478zzrrZZZhjhhjhhjjhhffGGhjjhiuuijkjjkihjhjh}
\begin{cases}
F_{00}=0\\ F_{0j}=-F_{j0}=E_j\quad\quad\forall\, j=1,2,3\\
F_{jj}=0\quad\quad\forall\, j=1,2,3
\\
F_{12}=-F_{21}=-B_3
\\
F_{13}=-F_{31}=B_2
\\
F_{23}=-F_{32}=-B_1\,,
\end{cases}
\end{equation}
where  $\vec E:=(E_1,E_2,E_3)$ and $\vec B:=(B_1,B_2,B_3)$, then by
inserting
\er{huohuioy89gjjhjffffff3478zzrrZZZhjhhjhhjjhhffGGhjjhiuuijkjjkihjhjh}
into
\er{khjhhkfgjfjhgghhgjghjhjkkkkgjghghuiiiulkkjlkklKKgfgjhjjghgjhhjhj}
we deduce:
\begin{equation}\label{khjhhkfgjfjhgghhgjghjhjkkkkgjghghuiiiulkkjlkklKKgfgjhjjghgjhhjhjhhhhh}
\begin{cases}
F^{00}=0
\\
F^{jj}=0 \quad\forall\, j=1,2,3,
\\
F^{01}=-F^{10}=-F_{01}+\frac{v^2}{c}F_{12}+\frac{v^3}{c}F_{13}=-\left(E_1+\frac{1}{c}\left(v^2B_3-v^3B_2\right)\right)\\
F^{02}=-F^{20}=-F_{02}+\frac{v^1}{c}F_{21}+\frac{v^3}{c}F_{23}=-\left(E_2+\frac{1}{c}\left(v^3B_1-v^1B_3\right)\right)\\
F^{03}=-F^{30}=-F_{03}+\frac{v^1}{c}F_{31}+\frac{v^2}{c}F_{32}=-\left(E_3+\frac{1}{c}\left(v^1B_2-v^2B_1\right)\right)
\\

F^{12}=-F^{21}=\frac{v^1}{c}F^{02}-\frac{v^2}{c}F^{01}+F_{12}=-\left(B_3+\frac{1}{c}\left(v^1F^{20}-v^2F^{10}\right)\right)\\
F^{13}=-F^{31}=\frac{v^1}{c}F^{03}-\frac{v^3}{c}F^{01}+F_{13}=B_2+\frac{1}{c}\left(v^3F^{10}-v^1F^{30}\right)\\
F^{23}=-F^{32}=\frac{v^2}{c}F^{03}-\frac{v^3}{c}F^{02}+F_{23}=-\left(B_1+\frac{1}{c}\left(v^2F^{30}-v^3F^{20}\right)\right).
\end{cases}
\end{equation}
Thus, as before in \er{MaxVacFullPPN}, denoting:
\begin{equation}\label{MaxVacFullPPNhjjgh}
\begin{cases}
\vec D:=\vec E+\frac{1}{c}\,\vec v\times
\vec B\\
\vec H:=\vec B+\frac{1}{c}\,\vec v\times \vec D,
\end{cases}
\end{equation}
and denoting $\vec D:=(D_1,D_2,D_3)$ and $\vec H:=(H_1,H_2,H_3)$ we
rewrite
\er{khjhhkfgjfjhgghhgjghjhjkkkkgjghghuiiiulkkjlkklKKgfgjhjjghgjhhjhjhhhhh}
as:
\begin{equation}\label{khjhhkfgjfjhgghhgjghjhjkkkkgjghghuiiiulkkjlkklKKgfgjhjjghgjhhjhjhhhhhghhg}
\begin{cases}
F^{00}=0
\\
F^{0j}=-F^{j0}=-D_j\quad\forall\, j=1,2,3,\\
F^{jj}=0 \quad\forall\, j=1,2,3,\\
F^{12}=-F^{21}=-H_3\\
F^{13}=-F^{31}=H_2\\
F^{23}=-F^{32}=-H_1.
\end{cases}
\end{equation}
In particular, by
\er{huohuioy89gjjhjffffff3478zzrrZZZhjhhjhhjjhhffGGhjjhiuuijkjjkihjhjh}
and
\er{khjhhkfgjfjhgghhgjghjhjkkkkgjghghuiiiulkkjlkklKKgfgjhjjghgjhhjhjhhhhhghhg},
using \er{MaxVacFullPPNhjjgh} we deduce that the scalar field on the
group $\mathcal{S}_0$: $L_e$, defined as:
\begin{equation}\label{MaxVacFullPPNhjjghjjkjhh}
L_e:=\sum_{j=0}^{3}\sum_{k=0}^{3}F^{jk}F_{jk},
\end{equation}
satisfies
\begin{multline}\label{MaxVacFullPPNhjjghjjkjhh1}
L_e=F^{00}F_{00}+\sum_{k=1}^{3}F^{0k}F_{0k}+\sum_{j=1}^{3}F^{j0}F_{j0}+\sum_{j=1}^{3}\sum_{k=1}^{3}F^{jk}F_{jk}=-2\vec
E\cdot\vec D+2\vec B\cdot\vec H=\\-2\left(\left(\vec
D-\frac{1}{c}\,\vec v\times \vec B\right)\cdot\vec D-\vec
B\cdot\left(\vec B+\frac{1}{c}\,\vec v\times \vec
D\right)\right)=-2\left(|\vec D|^2-|\vec B|^2\right).
\end{multline}
\subsection{Maxwell equations in covariant formulation}
It is well known from Tensor Analysys that if $\{S^{ij}\}_{0\leq
i,j\leq 3}$ is the antisymmetric two times contravariant tensor and
if $\{\xi_{ij}\}_{0\leq i,j\leq 3}$ is a symmetric two times
covariant and non-degenerate tensor, both on the certain group
$\mathcal{S}$, then the four-component field $\{\theta_{k}\}_{0\leq
k\leq 3}$ defined by
\begin{equation}\label{MaxVacFullPPNhjjghjjkjhhoujii}
\theta_k:=\sum_{j=0}^{3}\frac{\partial S^{kj}}{\partial
x^j}+\sum_{j=0}^{3}\frac{S^{kj}}{\sqrt{|\text{det}\,\xi|}}\frac{\partial}{\partial
x^j}\left(\sqrt{|\text{det}\,\xi|}\right)\quad\quad\forall\,
k=0,1,2,3,
\end{equation}
is a four-vector on $\mathcal{S}$. Here $\xi$ is a $4\times
4$-matrix defined by:
\begin{equation}\label{fgjfjhgghhgjghjhjkkkkgjghghuiiiuujhjhjkljjhjhjji}
\xi=\{\xi_{ij}\}_{0\leq i,j\leq 3}.
\end{equation}

In particular, if we consider the $4\times 4$-matrix $G$ defined by
\er{fgjfjhgghhgjghjhjkkkkgjghghuiiiuujhjhjklj} as:
\begin{equation}\label{fgjfjhgghhgjghjhjkkkkgjghghuiiiuujhjhjklj1}
G=\{g_{ij}\}_{0\leq i,j\leq 3},
\end{equation}
that satisfies \er{fgjfjhgghhgjghjhjkkkkgjghghuiiiuujhjh1} in every
cartesian coordinate system, i.e.
\begin{equation}\label{fgjfjhgghhgjghjhjkkkkgjghghuiiiuujhjhhjjg}
\text{det}\,G=-1.
\end{equation}
then for the lifted contravariant tensor of the electromagnetic
field $\{F^{ij}\}_{0\leq i,j\leq 3}$ on the group $\mathcal{S}_0$,
considered in
\er{khjhhkfgjfjhgghhgjghjhjkkkkgjghghuiiiulkkjlkklKKgfgjhjjghgjhhjhjhhhhhghhg},
as in \er{MaxVacFullPPNhjjghjjkjhhoujii} we can define the
four-vector field:
\begin{equation}\label{MaxVacFullPPNhjjghjjkjhhoujiiikjjihjhiuiu}
\left\{\sum_{j=0}^{3}\frac{\partial F^{kj}}{\partial
x^j}+\sum_{j=0}^{3}\frac{F^{kj}}{\sqrt{|\text{det}\,G|}}\frac{\partial}{\partial
x^j}\left(\sqrt{|\text{det}\,G|}\right)\right\}_{0\leq k\leq
3}=\left\{\sum_{j=0}^{3}\frac{\partial F^{kj}}{\partial
x^j}\right\}_{0\leq k\leq 3}
\end{equation}
on the group $\mathcal{S}_0$. Note here that we denoted the matrix
$G=\{g_{ij}\}_{0\leq i,j\leq 3}$ by the same letter as the
Gravitational Constant $G$. However, there is no ambiguity, since in
the second case $G$ is a constant scalar and in the first case $G$
is a matrix. Moreover, we will use the matrix notation
$G=\{g_{ij}\}_{0\leq i,j\leq 3}$ only in the expressions containing
term $\text{det}\,G$. Then by
\er{khjhhkfgjfjhgghhgjghjhjkkkkgjghghuiiiulkkjlkklKKgfgjhjjghgjhhjhjhhhhhghhg},
denoting
$$(x^0,x^1,x^2,x^3):=(ct,x_1,x_2,x_3)=(ct,\vec x),$$
we deduce:
\begin{equation}\label{khjhhkfgjfjhgghhgjghjhjkkkkgjghghuiiiulkkjlkklKKgfgjhjjghgjhhjhjhhhhhghhgtyytojjj}
\begin{cases}
\sum_{j=0}^{3}\frac{\partial F^{0j}}{\partial x^j}=-div_{\vec x}\vec D\\
\sum_{j=0}^{3}\frac{\partial F^{1j}}{\partial x^j}=\frac{1}{c}\frac{\partial D_1}{\partial t}-\left(\frac{\partial H_3}{\partial x_2}-\frac{\partial H_2}{\partial x_3}\right)\\
\sum_{j=0}^{3}\frac{\partial F^{2j}}{\partial x^j}=\frac{1}{c}\frac{\partial D_2}{\partial t}-\left(\frac{\partial H_1}{\partial x_3}-\frac{\partial H_3}{\partial x_1}\right)\\
\sum_{j=0}^{3}\frac{\partial F^{3j}}{\partial
x^j}=\frac{1}{c}\frac{\partial D_3}{\partial t}-\left(\frac{\partial
H_2}{\partial x_1}-\frac{\partial H_1}{\partial x_2}\right).
\end{cases}
\end{equation}
I.e.:
\begin{equation}\label{khjhhkfgjfjhgghhgjghjhjkkkkgjghghuiiiulkkjlkklKKgfgjhjjghgjhhjhjhhhhhghhgtyytojjjjj}
\left(\sum_{j=0}^{3}\frac{\partial F^{0j}}{\partial
x^j},\sum_{j=0}^{3}\frac{\partial F^{1j}}{\partial
x^j},\sum_{j=0}^{3}\frac{\partial F^{2j}}{\partial
x^j},\sum_{j=0}^{3}\frac{\partial F^{3j}}{\partial
x^j}\right)=\left(-div_{\vec x}\vec D,
\left(\frac{1}{c}\frac{\partial \vec D}{\partial t}-curl_{\vec
x}\vec H\right)\right).
\end{equation}
Therefore, by
\er{khjhhkfgjfjhgghhgjghjhjkkkkgjghghuiiiulkkjlkklKKgfgjhjjghgjhhjhjhhhhhghhgtyytojjjjj},
the first pair of Maxwell Equations in \er{MaxVacFullPPN}:
\begin{equation}\label{MaxVacFullPPNhjjghyghghiyyhhhj}
\begin{cases}
curl_{\vec x} \vec H=\frac{4\pi}{c}\vec j+\frac{1}{c}\frac{\partial \vec D}{\partial t}\\
div_{\vec x} \vec D= 4\pi\rho,
\end{cases}
\end{equation}
is equivalent to the following equations:
\begin{equation}\label{khjhhkfgjfjhgghhgjghjhjkkkkgjghghuiiiulkkjlkklKKgfgjhjjghgjhhjhjhhhhhghhgtyytojjjjjuyiy}
\left(\sum_{j=0}^{3}\frac{\partial F^{0j}}{\partial
x^j},\sum_{j=0}^{3}\frac{\partial F^{1j}}{\partial
x^j},\sum_{j=0}^{3}\frac{\partial F^{2j}}{\partial
x^j},\sum_{j=0}^{3}\frac{\partial F^{3j}}{\partial
x^j}\right)
=-4\pi(j^0,j^1,j^2,j^3),
\end{equation}
where $(j^0,j^1,j^2,j^3)$ is the four-vector of electric current on
the group $\mathcal{S}_0$ defined by
\er{fgjfjhgghhgjghjhjijhojihjhjjijhjjjjjuiijjjklihhojjjo} as:
\begin{equation}
\label{fgjfjhgghhgjghjhjijhojihjhjjijhjjjjjuiijjjklihhojjjoouuoiuiu}
(j^0,j^1,j^2,j^3):=\left(\rho,\frac{1}{c}\,\vec j\right)
\end{equation}
Note that in both sides of equation
\er{khjhhkfgjfjhgghhgjghjhjkkkkgjghghuiiiulkkjlkklKKgfgjhjjghgjhhjhjhhhhhghhgtyytojjjjjuyiy}
we have four-vectors and thus
\er{khjhhkfgjfjhgghhgjghjhjkkkkgjghghuiiiulkkjlkklKKgfgjhjjghgjhhjhjhhhhhghhgtyytojjjjjuyiy}
is a covariant form of \er{MaxVacFullPPNhjjghyghghiyyhhhj}. On the
other hand, the second pair of Maxwell Equations in
\er{MaxVacFullPPN}:
\begin{equation}\label{MaxVacFullPPNhjjghyghghiyy}
\begin{cases}
curl_{\vec x} \vec E+\frac{1}{c}\frac{\partial \vec B}{\partial t}=0\\
div_{\vec x} \vec B=0,
\end{cases}
\end{equation}
is equivalent to
\er{MaxVacFull1bjkgjhjhgjgjgkjfhjfdghghligioiuittrPPNkkk}, i.e. to
the following:
\begin{equation}\label{MaxVacFull1bjkgjhjhgjgjgkjfhjfdghghligioiuittrPPNkkkouiuiuuiui}
\begin{cases}
\vec B= curl_{\vec x} \vec A,\\
\vec E=-\nabla_{\vec x}\Psi-\frac{1}{c}\frac{\partial\vec
A}{\partial t},
\end{cases}
\end{equation}
On the other hand, as before, by
\er{huohuioy89gjjhjffffff3478zzrrZZZhjhhjhhjjhhffGGhjjhiuuijkjjk} we
can rewrite
\er{MaxVacFull1bjkgjhjhgjgjgkjfhjfdghghligioiuittrPPNkkkouiuiuuiui}
in the form of
\er{huohuioy89gjjhjffffff3478zzrrZZZhjhhjhhjjhhffGGhjjh}:
\begin{equation}\label{huohuioy89gjjhjffffff3478zzrrZZZhjhhjhhjjhhffGGhjjhuyuyu}
F_{ij}=\frac{\partial A_j}{\partial x^i}-\frac{\partial
A_i}{\partial x^j}\quad\quad\forall\, i,j=0,1,2,3\,,
\end{equation}
where $(A_0,A_1,A_2,A_3)$ is the four-covector of the
electromagnetic potential on the group $\mathcal{S}_0$ defined by
\er{fgjfjhgghhgjghjhjijhojihjhjjijhjjjljljpk} as:
\begin{equation}\label{fgjfjhgghhgjghjhjijhojihjhjjijhjjjljljpkyuuyyu}
(A_0,A_1,A_2,A_3)=(\Psi,-\vec A).
\end{equation}
Note that in both sides of equation
\er{huohuioy89gjjhjffffff3478zzrrZZZhjhhjhhjjhhffGGhjjhuyuyu} we
have two time covariant tensors, and thus
\er{huohuioy89gjjhjffffff3478zzrrZZZhjhhjhhjjhhffGGhjjhuyuyu} is a
covariant form of \er{MaxVacFullPPNhjjghyghghiyy}. Finally, by
\er{khjhhkfgjfjhgghhgjghjhjkkkkgjghghuiiiulkkjlkklKKgfgjhjjghgjhhjhjhhhhh},
the relations between $(\vec E,\vec B)$ and $(\vec D,\vec H)$ in
\er{MaxVacFullPPNhjjgh}:
\begin{equation}\label{MaxVacFullPPNhjjghojjkjkiou}
\begin{cases}
\vec D=\vec E+\frac{1}{c}\,\vec v\times
\vec B\\
\vec H=\vec B+\frac{1}{c}\,\vec v\times \vec D,
\end{cases}
\end{equation}
are equivalent to the following covariant equations:
\begin{equation}\label{fgjfjhgghhgjghjhjkkkkgjghghuiiiulkkjlkklKKkjkj}
F^{mn}:=\sum_{k=0}^{3}\sum_{j=0}^{3}g^{mj}g^{nk}F_{jk}
\quad\quad\forall\, m,n=0,1,2,3.
\end{equation}
Thus by
\er{huohuioy89gjjhjffffff3478zzrrZZZhjhhjhhjjhhffGGhjjhuyuyu},
\er{fgjfjhgghhgjghjhjkkkkgjghghuiiiulkkjlkklKKkjkj} and
\er{khjhhkfgjfjhgghhgjghjhjkkkkgjghghuiiiulkkjlkklKKgfgjhjjghgjhhjhjhhhhhghhgtyytojjjjjuyiy}
together, we deduce that the full system of Maxwell Equations in
\er{MaxVacFullPPN}:
\begin{equation}\label{MaxVacFullPPNhjjghyghghiyyhh}
\begin{cases}
curl_{\vec x} \vec H=\frac{4\pi}{c}\vec j+\frac{1}{c}\frac{\partial \vec D}{\partial t}\\
div_{\vec x} \vec D=4\pi\rho\\
curl_{\vec x} \vec E+\frac{1}{c}\frac{\partial \vec B}{\partial t}=0\\
div_{\vec x} \vec B=
0\\
\vec E=\vec D-\frac{1}{c}\,\vec v\times
\vec B\\
\vec H=\vec B+\frac{1}{c}\,\vec v\times \vec D,
\end{cases}
\end{equation}
is equivalent to the following covariant equations:
\begin{equation}\label{khjhhkfgjfjhgghhgjghjhjkkkkgjghghuiiiulkkjlkklKKgfgjhjjghgjhhjhjhhhhhghhgtyytojjjjjuyiyuu}
\sum_{j=0}^{3}\frac{\partial}{\partial
x^j}\left(\sum_{m=0}^{3}\sum_{n=0}^{3}g^{km}g^{jn}\left(\frac{\partial
A_n}{\partial x^m}-\frac{\partial A_m}{\partial
x^n}\right)\right)=-4\pi j^k\quad\quad\forall\, k=0,1,2,3.
\end{equation}
Note that equations
\er{khjhhkfgjfjhgghhgjghjhjkkkkgjghghuiiiulkkjlkklKKgfgjhjjghgjhhjhjhhhhhghhgtyytojjjjjuyiyuu}
are fully analogous to the covariant formulation of Maxwell
equations in Special Relativity and the only difference is the
choice of the pseudo-metric tensor $\{g^{ij}\}_{0\leq i,j\leq 3}$
(Note that for the Special Relativity case we also have
$\text{det}\,G=-1$). As for the cases of the General relativity, the
covariant formulation of Maxwell equations is still similar to
\er{khjhhkfgjfjhgghhgjghjhjkkkkgjghghuiiiulkkjlkklKKgfgjhjjghgjhhjhjhhhhhghhgtyytojjjjjuyiyuu},
however, in addition to the different choice of the pseudo-metric
tensor $\{g^{ij}\}_{0\leq i,j\leq 3}$ we also have
$\text{det}\,G\neq Const.$ and thus for the full analogy equations
\er{khjhhkfgjfjhgghhgjghjhjkkkkgjghghuiiiulkkjlkklKKgfgjhjjghgjhhjhjhhhhhghhgtyytojjjjjuyiyuu}
should be rewritten in the enlarged form, due to
\er{MaxVacFullPPNhjjghjjkjhhoujii},\er{MaxVacFullPPNhjjghjjkjhhoujiiikjjihjhiuiu}:
\begin{multline}\label{khjhhkfgjfjhgghhgjghjhjkkkkgjghghuiiiulkkjlkklKKgfgjhjjghgjhhjhjhhhhhghhgtyytojjjjjuyiyuughgfg}
\sum_{j=0}^{3}\frac{\partial}{\partial
x^j}\left(\sum_{m=0}^{3}\sum_{n=0}^{3}g^{km}g^{jn}\left(\frac{\partial
A_n}{\partial x^m}-\frac{\partial A_m}{\partial
x^n}\right)\right)+\\
\sum_{j=0}^{3}\frac{1}{\sqrt{|\text{det}\,G|}}\frac{\partial}{\partial
x^j}\left(\sqrt{|\text{det}\,G|}\right)\left(\sum_{m=0}^{3}\sum_{n=0}^{3}g^{km}g^{jn}\left(\frac{\partial
A_n}{\partial x^m}-\frac{\partial A_m}{\partial x^n}\right)\right)
=-4\pi j^k\quad\quad\forall\, k=0,1,2,3.
\end{multline}
Note also that we can rewrite
\er{khjhhkfgjfjhgghhgjghjhjkkkkgjghghuiiiulkkjlkklKKgfgjhjjghgjhhjhjhhhhhghhgtyytojjjjjuyiyuughgfg}
as:
\begin{equation}\label{khjhhkfgjfjhgghhgjghjhjkkkkgjghghuiiiulkkjlkklKKgfgjhjjghgjhhjhjhhhhhghhgtyytojjjjjuyiyuughgfghhj}
\sum_{j=0}^{3}\frac{\partial}{\partial
x^j}\left(\sum_{m=0}^{3}\sum_{n=0}^{3}\sqrt{|\text{det}\,G|}\,g^{km}g^{jn}\left(\frac{\partial
A_n}{\partial x^m}-\frac{\partial A_m}{\partial x^n}\right)\right)
=-4\pi \sqrt{|\text{det}\,G|}\, j^k\quad\quad\forall\, k=0,1,2,3.
\end{equation}

Next by \er{MaxVacFullPPNhjjghjjkjhh} and
\er{MaxVacFullPPNhjjghjjkjhh1} we have
\begin{equation}\label{MaxVacFullPPNhjjghjjkjhhiuyy}
\frac{1}{2}|\vec D|^2-\frac{1}{2}|\vec
B|^2=-\sum_{j=0}^{3}\sum_{k=0}^{3}\frac{1}{4}F^{jk}F_{jk}.
\end{equation}
Therefore, by
\er{fgjfjhgghhgjghjhjijhojihjhjjijhjjjjjuiijjjklihhojjjoouuoiuiu},
\er{fgjfjhgghhgjghjhjijhojihjhjjijhjjjljljpkyuuyyu} and
\er{MaxVacFullPPNhjjghjjkjhhiuyy}, we can rewrite the density of the
Lagrangian of the electromagnetic field, defined in
\er{vhfffngghkjgghPPNggjgjjkgj} as
\begin{equation}\label{vhfffngghkjgghPPNggjgjjkgjoiui}
L_1\left(\vec A,\Psi,\vec
x,t\right):=\frac{1}{4\pi}\left(\frac{1}{2}\left|\vec
D\right|^2-\frac{1}{2}\left|\vec
B\right|^2-4\pi\left(\rho\Psi-\frac{1}{c}\vec A\cdot\vec
j\right)\right),
\end{equation}
in the equivalent covariant form:
\begin{multline}\label{vhfffngghkjgghPPNggjgjjkgjoiuioiggh}
L_1=\frac{1}{4\pi}\left(-\sum_{n=0}^{3}\sum_{k=0}^{3}\frac{1}{4}F^{nk}F_{nk}-\sum_{k=0}^{3}4\pi
j^k A_k\right)=\\
\frac{1}{4\pi}\left(-\sum_{n=0}^{3}\sum_{k=0}^{3}\sum_{m=0}^{3}\sum_{p=0}^{3}\frac{1}{4}g^{mn}g^{pk}\left(\frac{\partial
A_p}{\partial x^m}-\frac{\partial A_m}{\partial
x^p}\right)\left(\frac{\partial A_k}{\partial x^n}-\frac{\partial
A_n}{\partial x^k}\right)-\sum_{k=0}^{3}4\pi j^k A_k\right).
\end{multline}
The density of Lagrangian in
\er{vhfffngghkjgghPPNggjgjjkgjoiuioiggh} is also fully analogous to
the covariant formulation of the Lagrangian density of the
electromagnetic field in Special and General Relativity and the only
difference is the choice of the pseudo-metric tensor
$\{g^{ij}\}_{0\leq i,j\leq 3}$.

\subsection{Covariant formulation of Lagrangian of motion of a classical
charged particle in the external gravitational and electromagnetic
fields} Given a classical charged particle with inertial mass $m$,
charge $\sigma$, three-dimensional place $\vec r(t)$ and
three-dimensional velocity $\frac{d\vec r}{dt}$ in the outer
gravitational field with three-dimensional vectorial potential $\vec
v(\vec x,t)$, the outer electromagnetical field with
three-dimensional vectorial potential $\vec A(\vec x,t)$ and scalar
potential $\vec \Psi(\vec x,t)$, consider a usual Lagrangian that is
a particular case of \er{vhfffngghkjgghfjjSYSPN}:
\begin{equation}\label{vhfffngghkjgghfjjSYSPNkoijjhpoi}
L_0\left(\frac{d\vec r}{dt},t\right):=
\left\{\frac{m}{2}\left|\frac{d\vec r}{dt}-\vec v(\vec
r,t)\right|^2-\sigma\left(\Psi(\vec r,t)-\frac{1}{c}\vec A(\vec
r,t)\cdot\frac{d\vec r}{dt}\right)\right\}.
\end{equation}
Then, since we are interesting in critical points of the functional
\begin{equation}\label{btfffygtgyggyijhhkkSYSPNuhuygyygyggy}
J_0=\int_0^T L_0\left(\frac{d\vec r}{dt},\vec r,t\right)dt,
\end{equation}
adding a constant does not changes the physical meaning of the
Lagrangian and we can rewrite \er{vhfffngghkjgghfjjSYSPNkoijjhpoi}
as:
\begin{equation}\label{vhfffngghkjgghfjjSYSPNkoijjhpoiuui}
L'_0\left(\frac{d\vec r}{dt},t\right):=
\left\{\left(\frac{m}{2}\left|\frac{d\vec r}{dt}-\vec v(\vec
r,t)\right|^2-\frac{mc^2}{2}\right)-\sigma\left(\Psi(\vec
r,t)-\frac{1}{c}\vec A(\vec r,t)\cdot\frac{d\vec
r}{dt}\right)\right\}.
\end{equation}
and \er{btfffygtgyggyijhhkkSYSPNuhuygyygyggy} as
\begin{multline}\label{btfffygtgyggyijhhkkSYSPNuhuygyygyggyuyy}
J'_0:=J_0-\frac{Tmc^2}{2}=\int_0^T L'_0\left(\frac{d\vec r}{dt},\vec
r,t\right)dt=\\ \int_0^T\left\{\left(\frac{m}{2}\left|\frac{d\vec
r}{dt}-\vec v(\vec
r,t)\right|^2-\frac{mc^2}{2}\right)-\sigma\left(\Psi(\vec
r,t)-\frac{1}{c}\vec A(\vec r,t)\cdot\frac{d\vec
r}{dt}\right)\right\}dt,
\end{multline}
Next consider the four-vector field of the momentum on the group
$\mathcal{S}_0$: $\left(p^0(t),p^1(t),p^2(t),p^3(t)\right)$, defined
by \er{fgjfjhgghhgjghjhjijhojihjhjjijhjjjjjuiijhjhh} and
\er{fgjfjhgghhgjghjhjijhojihjhjjijhjjjjjuiijjjklihh} as:
\begin{equation}\label{fgjfjhgghhgjghjhjijhojihjhjjijhjjjjjuiijhjhhioi}
\left(
p^0(t),p^1(t),p^2(t),p^3(t)\right):=\left(m,\frac{m}{c}\frac{d\vec
r}{dt}(t)\right)=
\left(m,\frac{m}{c}\frac{dr_1}{dt}(t),\frac{m}{c}\frac{dr_2}{dt}(t),\frac{m}{c}\frac{dr_3}{dt}(t)\right)
\end{equation}
Then by
\er{fgjfjhgghhgjghjhjkkkkgjghghuiiiulkkjKKyuyyu0ioioiogghghghgghghhg}
and
\er{fgjfjhgghhgjghjhjkkkkgjghghuiiiulkkjKKyuyyu0ioioiogghghghgghgh}
we have
\begin{multline}\label{fgjfjhgghhgjghjhjkkkkgjghghuiiiulkkjKKyuyyu0ioioiogghghghgghghhghgv}
\frac{mc^2}{2}\left(\frac{1}{c^2}\left|\frac{d\vec r}{dt}-\vec
v(\vec r,t)\right|^2-1\right)=\frac{m}{2}\left|\frac{d\vec
r}{dt}-\vec v(\vec
r,t)\right|^2-\frac{mc^2}{2}\\=-\frac{c^2}{2m}\left(\sum_{k=0}^{3}p^kp_k\right)=-\frac{mc^2}{2}\left(\sum_{j=0}^{3}\sum_{k=0}^{3}g_{jk}(\vec
r,t)\,\frac{p^j}{m}\,\frac{p^k}{m}\right).
\end{multline}
On the other hand if we consider the four-covector of the
electromagnetic potential on the group $\mathcal{S}_0$:
$(A_0,A_1,A_2,A_3)$, defined by
\er{fgjfjhgghhgjghjhjijhojihjhjjijhjjjljljpk} as:
\begin{equation}\label{fgjfjhgghhgjghjhjijhojihjhjjijhjjjljljpkyuuyyuhhhhj}
(A_0,A_1,A_2,A_3)=(\Psi,-\vec A),
\end{equation}
then we can write,
\begin{equation}\label{btfffygtgyggyijhhkkSYSPNuhuygyygyggyuyyioiooiihhhj}
\sigma\left(\Psi(\vec r,t)-\frac{1}{c}\vec A(\vec
r,t)\cdot\frac{d\vec r}{dt}\right)=\sum_{k=0}^{3}\sigma A_k(\vec
r,t)\,\frac{p^k}{m}.
\end{equation}
Thus by
\er{fgjfjhgghhgjghjhjkkkkgjghghuiiiulkkjKKyuyyu0ioioiogghghghgghghhghgv}
and \er{btfffygtgyggyijhhkkSYSPNuhuygyygyggyuyyioiooiihhhj} we
rewrite \er{btfffygtgyggyijhhkkSYSPNuhuygyygyggyuyy} in a covariant
form:
\begin{multline}\label{btfffygtgyggyijhhkkSYSPNuhuygyygyggyuyyuyuy}
J'_0=\int_0^T L'_0\left(\frac{d\vec r}{dt},\vec r,t\right)dt=
\int_0^T\left\{-\frac{mc^2}{2}\left(\sum_{j=0}^{3}\sum_{k=0}^{3}g_{jk}(\vec
r,t)\,\frac{p^j}{m}\,\frac{p^k}{m}\right)-\sum_{k=0}^{3}\sigma
A_k(\vec r,t)\,\frac{p^k}{m}\right\}dt.
\end{multline}
Thus if we consider the four-dimensional space-time trajectory of
the particle:
\begin{equation}\label{btfffjhgjghgh}
\left(\chi^0(t),\chi^1(t),\chi^2(t),\chi^3(t)\right)=\left(t,\frac{1}{c}r_1(t),\frac{1}{c}r_2(t),\frac{1}{c}r_3(t)\right),
\end{equation}
then we rewrite \er{btfffygtgyggyijhhkkSYSPNuhuygyygyggyuyyuyuy} as:
\begin{equation}\label{btfffygtgyggyijhhkkSYSPNuhuygyygyggyuyyuyuykuhghg}
J'_0=
\int_0^T\left\{-\frac{mc^2}{2}\left(\sum_{j=0}^{3}\sum_{k=0}^{3}g_{jk}\left(\chi(t)\right)\,\frac{d\chi^j}{dt}\,\frac{d\chi^k}{dt}\right)-\sum_{k=0}^{3}\sigma
A_k\left(\chi(t)\right)\,\frac{d\chi^k}{dt}\right\}dt.
\end{equation}
Moreover, $\left(\frac{d\chi^0}{dt}, \frac{d\chi^1}{dt},
\frac{d\chi^2}{dt}, \frac{d\chi^3}{dt}\right)$ is a four-vector on
the group $\mathcal{S}_0$ and the global non-relativistic time $t$
is the scalar on the group $\mathcal{S}_0$.

Next we also can consider a more general Lagrangian than
\er{btfffygtgyggyijhhkkSYSPNuhuygyygyggyuyyuyuykuhghg}: given a
function $\mathcal{G}(\tau):\mathbb{R}\to\mathbb{R}$ define:
\begin{equation}\label{btfffygtgyggyijhhkkSYSPNuhuygyygyggyuyyuyuykuhghgjjoj}
J_{\mathcal{G}}(\chi)=
\int_0^T\left\{-mc^2\;\mathcal{G}\left(\sum_{j=0}^{3}\sum_{k=0}^{3}g_{jk}\left(\chi(t)\right)\,\frac{d\chi^j}{dt}\,\frac{d\chi^k}{dt}\right)-\sum_{k=0}^{3}\sigma
A_k\left(\chi(t)\right)\,\frac{d\chi^k}{dt}\right\}dt.
\end{equation}
Clearly, \er{btfffygtgyggyijhhkkSYSPNuhuygyygyggyuyyuyuykuhghgjjoj}
is written in covariant form, and in particular,
\er{btfffygtgyggyijhhkkSYSPNuhuygyygyggyuyyuyuykuhghgjjoj} is
invariant under the change of non-inertial cartesian coordinate
systems. In particular, for $\mathcal{G}(\tau):=\frac{1}{2}\tau$ we
obtain \er{btfffygtgyggyijhhkkSYSPNuhuygyygyggyuyyuyuykuhghg}.

Another important particular case is the following choice:
$\mathcal{G}(\tau):=\sqrt{\tau}$. Then we deduce:
\begin{equation}\label{btfffygtgyggyijhhkkSYSPNuhuygyygyggyuyyuyuykuhghgjjojiyy}
J_{rl}(\chi)=
\int_0^T\left\{-mc^2\;\sqrt{\left(\sum_{j=0}^{3}\sum_{k=0}^{3}g_{jk}\left(\chi(t)\right)\,\frac{d\chi^j}{dt}\,\frac{d\chi^k}{dt}\right)}-\sum_{k=0}^{3}\sigma
A_k\left(\chi(t)\right)\,\frac{d\chi^k}{dt}\right\}dt,
\end{equation}
that is in somewhat analogous to the relativistic Lagrangian of the
motion of charged particle. Due to \er{btfffjhgjghgh} we rewrite
\er{btfffygtgyggyijhhkkSYSPNuhuygyygyggyuyyuyuykuhghgjjojiyy} in a
three-dimensional form as:
\begin{equation}\label{btfffygtgyggyijhhkkSYSPNuhuygyygyggyuyyuyuykuhghgjjojiyyyuyuuyy}
J_{rl}(\vec r)= \int_0^T\left\{
-mc^2\sqrt{1-\frac{1}{c^2}\left|\frac{d\vec r}{dt}-\vec v(\vec
r,t)\right|^2}-\sigma\left(\Psi(\vec r,t)-\frac{1}{c}\vec A(\vec
r,t)\cdot\frac{d\vec r}{dt}\right)\right\}dt.
\end{equation}
Thus in the case
$$\frac{1}{c^2}\left|\frac{d\vec r}{dt}-\vec v(\vec
r,t)\right|^2\ll 1,$$ up to additive constant,
\er{btfffygtgyggyijhhkkSYSPNuhuygyygyggyuyyuyuykuhghgjjojiyyyuyuuyy}
becomes to be \er{btfffygtgyggyijhhkkSYSPNuhuygyygyggy}, where $L_0$
is given by \er{vhfffngghkjgghfjjSYSPNkoijjhpoi}. Note that the
Lagrangian in
\er{btfffygtgyggyijhhkkSYSPNuhuygyygyggyuyyuyuykuhghgjjojiyy} has
the following advantage with respect to
\er{btfffygtgyggyijhhkkSYSPNuhuygyygyggyuyyuyuykuhghg}: if we
parameterize the curve in \er{btfffjhgjghgh} by some arbitrary
parameter $s$ with increasing mapping $t\leftrightarrow s$, which is
however can differ from the global time $t$, then changing variables
of integration in
\er{btfffygtgyggyijhhkkSYSPNuhuygyygyggyuyyuyuykuhghgjjojiyy} from
$t$ to $s$ gives:
\begin{equation}\label{btfffygtgyggyijhhkkSYSPNuhuygyygyggyuyyuyuykuhghgjjojiyyjljlgh}
J_{rl}(\chi)=
\int_a^b\left\{-mc^2\;\sqrt{\left(\sum_{j=0}^{3}\sum_{k=0}^{3}g_{jk}\left(\chi(s)\right)\,\frac{d\chi^j}{ds}\,\frac{d\chi^k}{ds}\right)}-\sum_{k=0}^{3}\sigma
A_k\left(\chi(s)\right)\,\frac{d\chi^k}{ds}\right\}ds,
\end{equation}
that has exactly the same form as
\er{btfffygtgyggyijhhkkSYSPNuhuygyygyggyuyyuyuykuhghgjjojiyy},
however $s$ in
\er{btfffygtgyggyijhhkkSYSPNuhuygyygyggyuyyuyuykuhghgjjojiyyjljlgh}
can be \underline{arbitrary} parameter of the curve with increasing
mapping $t\leftrightarrow s$.

Finally, we would like to note that if the motion of some particle
is ruled by the relativistic-like Lagrangian in
\er{btfffygtgyggyijhhkkSYSPNuhuygyygyggyuyyuyuykuhghgjjojiyyyuyuuyy},
then, although the absolute value of the velocity of the particle
$\left|\frac{d\vec r}{dt}\right|$ can be arbitrary large, the
absolute value of the difference between the velocity of the
particle and the local gravitational potential cannot exceed the
value $c$, i.e.:
\begin{equation}\label{btfffygtgyggyijhhkkSYSPNuhuygyygyggyuyyuyuykuhghgjjojiyyyuyuuyyijyyuyu}
\left|\vec u(t)-\vec v(\vec r,t)\right|:=\left|\frac{d\vec
r}{dt}-\vec v(\vec r,t)\right|\,<\,c\quad\quad\quad\quad\forall\, t,
\end{equation}
provided that
\er{btfffygtgyggyijhhkkSYSPNuhuygyygyggyuyyuyuykuhghgjjojiyyyuyuuyyijyyuyu}
is satisfied in some initial instant of time. Note also that the
quantity in the left hand side of
\er{btfffygtgyggyijhhkkSYSPNuhuygyygyggyuyyuyuykuhghgjjojiyyyuyuuyyijyyuyu}
is invariant under the change of inertial or non-inertial cartesian
coordinate system.

\subsection{Kinematic pseudo-metric tensors of inertia}\label{jjjjjkkk88}
Consider $\{J^{ij}\}_{0\leq i,j\leq 3}$ to be a two times
contravariant tensor field on the group $\mathcal{S}_0$, defined by
\begin{equation}\label{hoyuiouigyfghgjh3478zzrrZZffGGjkkjojj88}
J^{ij}:=k^ik^j-{\Theta}^{ij}\quad\quad\forall\,i,j=0,1,2,3,
\end{equation}
where $\{{\Theta}^{ij}\}_{0\leq i,j\leq 3}$ is the contravariant
tensor of the three-dimensional geometry, defined by
\er{huohuioy89gjjhjffffff3478zzrrZZZhjhhjhhjjhhffGGhjjhuiiuihhjkkoioi}
and being a two times contravariant tensor, and $(k^0,k^1,k^2,k^3)$
is the four-dimensional potential of inertia, defined by
\er{fgjfjhgghhgjghjhjijhojihjhjjijhjjjjjuiijjjkjkkhhkkh} and being a
four-vector. Then by \er{fgjfjhgghhgjghjhjkkkkgjghghjljl} we obtain
that $\{J^{ij}\}_{0\leq i,j\leq 3}$ is indeed a two times
contravariant tensor field on the group $\mathcal{S}_0$ and
moreover, this tensor is symmetric. Moreover, by
\er{huohuioy89gjjhjffffff3478zzrrZZZhjhhjhhjjhhffGGhjjhuiiuihhjkkoioi}
and \er{fgjfjhgghhgjghjhjijhojihjhjjijhjjjjjuiijjjkjkkhhkkh} we
have:
\begin{equation}\label{hoyuiouigyfghgjh3478zzrrZZffGGjkkj88}
\begin{cases}
J^{00}=1\\
J^{ij}=-\delta_{ij}+\frac{k^ik^j}{c^2}\quad\forall 1\leq i,j\leq 3\\
J^{0j}=J^{j0}=\frac{k^j}{c}\quad\forall 1\leq j\leq 3,
\end{cases}
\end{equation}
where $\vec k=(k^1,k^2,k^3)$ is the three-dimensional vectorial
potential of inertia. We call the tensor $\{J^{ij}\}_{0\leq i,j\leq
3}$ the contravariant kinematic pseudo-metric tensor of inertia.
Next consider a $16$-component field $\{J_{ij}\}_{0\leq i,j\leq 3}$
defined by
\begin{equation}\label{hoyuiouigyfg3478zzrrZZffGGhhjhj88}
\begin{cases}
J_{00}=1-\frac{|\vec k|^2}{c^2}\\
J_{ij}=-\delta_{ij}\quad\forall 1\leq i,j\leq 3\\
J_{0j}=J_{j0}=\frac{k^j}{c}\quad\forall 1\leq j\leq 3.
\end{cases}
\end{equation}
Then, as before, we have
\begin{equation*}
\sum_{m=0}^{3}J_{0m}J^{m0}=J_{00}J^{00}+\sum_{m=1}^{3}J_{0m}J^{m0}=
1-\frac{|\vec k|^2}{c^2}+\frac{|\vec k|^2}{c^2}=1,
\end{equation*}
\begin{equation*}
\sum_{m=0}^{3}J_{im}J^{mj}=J_{i0}J^{0j}+\sum_{m=1}^{3}J_{im}J^{mj}=
\frac{k^ik^j}{c^2}+\delta_{ij}-\frac{k^ik^j}{c^2}=\delta_{ij}\quad\forall
1\leq i,j\leq 3,
\end{equation*}
and
\begin{equation*}
\sum_{m=0}^{3}J_{im}J^{m0}=J_{i0}J^{00}+\sum_{m=1}^{3}J_{im}J^{m0}=\frac{k^i}{c}-\frac{k^i}{c}=0
\quad\forall 1\leq i\leq 3,
\end{equation*}
\begin{multline*}
\sum_{m=0}^{3}J_{0m}J^{mj}=J_{00}J^{0j}+\sum_{m=1}^{3}J_{0m}J^{mj}=
\left(1-\frac{|\vec k|^2}{c^2}\right)\frac{k^j}{c}
-\sum_{m=1}^{3}\frac{k^m}{c}\left(\delta_{mj}-\frac{k^mk^j}{c^2}\right)=0
\quad\forall 1\leq j\leq 3,
\end{multline*}
where $\{J^{ij}\}_{0\leq i,j\leq 3}$ is the contravariant kinematic
pseudo-metric tensors of inertia, defined by
\er{hoyuiouigyfghgjh3478zzrrZZffGGjkkj88}. So,
\begin{equation}\label{fgjfjhgghhgjghjhjkkkkgjghghuiiiuujhjh88}
\sum_{m=0}^{3}J^{im}J_{mj}=\begin{cases}
1\quad\text{if}\quad i=j\\
0\quad\text{if}\quad i\neq j
\end{cases}\quad\quad\forall\, i,j=0,1,2,3.
\end{equation}
Therefore, by comparing \er{fgjfjhgghhgjghjhjkkkkgjghghuiiiuujhjh88}
and \er{fgjfjhgghhgjghjhjkkkkgjghghuiiiu} we deduce that
$\{J_{ij}\}_{i,j=0,1,2,3}$ is a two times covariant tensor on the
group $\mathcal{S}_0$, and moreover, this tensor is symmetric. We
call the tensor $\{J_{ij}\}_{0\leq i,j\leq 3}$ covariant kinematic
pseudo-metric tensors of inertia. Using
\er{fgjfjhgghhgjghjhjkkkkgjghghuiiiuujhjh88} we also obtain that the
pseudo-metric tensors $\{J_{ij}\}_{i,j=0,1,2,3}$ and
$\{J^{ij}\}_{0\leq i,j\leq 3}$ are non-degenerate and they are
inverse of each other. Moreover, as before, it can be easily
calculated that if we consider the $4\times 4$-matrix:
\begin{equation}\label{fgjfjhgghhgjghjhjkkkkgjghghuiiiuujhjhjklj88}
J=\{J_{ij}\}_{0\leq i,j\leq 3},
\end{equation}
then
\begin{equation}\label{fgjfjhgghhgjghjhjkkkkgjghghuiiiuujhjh188}
\text{det}\,J=-1.
\end{equation}
In particular, by \er{fgjfjhgghhgjghjhjkkkkgjghghuiiiuujhjh188} and
\er{fgjfjhgghhgjghjhjkkkkgjghghuiiiuujhjh1} we deduce
\begin{equation}\label{fgjfjhgghhgjghjhjkkkkgjghghuiiiuujhjhljhgugii88}
\text{det}\,G=\text{det}\,J.
\end{equation}
where $4\times 4$-matrix $G$ is given by,
\begin{equation}\label{fgjfjhgghhgjghjhjkkkkgjghghuiiiuujhjhjkljojuuiyi88}
G=\{g_{ij}\}_{0\leq i,j\leq 3},
\end{equation}
with the covariant pseudo-metric tensor of the four-dimensional
space-time $\{g_{ij}\}_{0\leq i,j\leq 3}$, given by
\er{hoyuiouigyfg3478zzrrZZffGGhhjhj}. Next, obviously we have
\begin{equation}\label{hoyuiouigyfghgjh3478zzrrZZffGGjkkjojj88iihhkhk88}
{g}^{ij}=J^{ij}+\left(v^iv^j-k^ik^j\right)\quad\quad\forall\,i,j=0,1,2,3,
\end{equation}
where $J^{ij}$ the is the contravariant kinematic pseudo-metric
tensor of inertia, given by
\er{hoyuiouigyfghgjh3478zzrrZZffGGjkkjojj88}, and $g^{ij}$ is the
contravariant pseudo-metric tensor of the four-dimensional
space-time, given by \er{hoyuiouigyfghgjh3478zzrrZZffGGjkkjojj}. In
particular, in the case of the absence of the genuine gravity where
$\vec v=\vec k$ we have
\begin{equation}\label{hoyuiouigyfghgjh3478zzrrZZffGGjkkjojj88iihhkhk88jouyiyi88}
{g}^{ij}=J^{ij}\quad\quad\forall\,i,j=0,1,2,3,
\end{equation}
or equivalently,
\begin{equation}\label{hoyuiouigyfghgjh3478zzrrZZffGGjkkjojj88iihhkhk88jouyiyi88joljljkj88}
{g}_{ij}=J_{ij}\quad\quad\forall\,i,j=0,1,2,3.
\end{equation}

Furthermore, by \er{hoyuiouigyfghgjh3478zzrrZZffGGjkkj88} we have
\begin{equation}\label{fgjfjhgghhgjghjhjijhojihjhjjijhjjjjjuiijjjkhjhjhjuiiuuuyu88jukukukuk88kpjljl}
\sum_{m=0}^{3} \,J^{jm}\,\left(c\,\frac{\partial\tau}{\partial
x^m}\right)=k^j\quad\quad\forall\,j=0,1,2,3.
\end{equation}
where $(k^0,k^1,k^2,k^3)$ is the four-dimensional potential of
inertia, defined by
\er{fgjfjhgghhgjghjhjijhojihjhjjijhjjjjjuiijjjkjkkhhkkh} and $\tau$
is the scalar of the global time on the group $\mathcal{S}_0$,
defined by \er{uguyytfddddgghkjhhjuuiui}. Thus since,
\begin{equation}\label{fgjfjhgghhgjghjhjijhojihjhjjijhjjjjjuiijjjkhjhjhjuiiuuuyu88jukukukuk88kpjljlkjoj}
\sum_{j=0}^{3} \,k^j\,\left(c\,\frac{\partial\tau}{\partial
x^j}\right)=1,
\end{equation}
we deduce the following eikonal type equation
\begin{equation}\label{fgjfjhgghhgjghjhjijhojihjhjjijhjjjjjuiijjjkhjhjhjuiiuuuyu88}
c^2\left(\sum_{j=0}^{3}\sum_{m=0}^{3}\,J^{jm}\frac{\partial\tau}{\partial
x^j}\frac{\partial\tau}{\partial x^m}\right)=1,
\end{equation}
which is similar to
\er{fgjfjhgghhgjghjhjijhojihjhjjijhjjjjjuiijjjkhjhjhjuiiuuuyu},
despite of the different pseudometrics. Moreover, by
\er{fgjfjhgghhgjghjhjijhojihjhjjijhjjjjjuiijjjkhjhjhjuiiuuuyu88jukukukuk88kpjljl}
and \er{fgjfjhgghhgjghjhjkkkkgjghghuiiiuujhjh88} we infer:
\begin{equation}\label{fgjfjhgghhgjghjhjijhojihjhjjijhjjjjjuiijjjkhjhjhjuiiuuuyu88jukukukuk88kpjljlmhmbbvnvvvkljljl88}
\sum_{m=0}^{3} \,J_{jm}\,k^m\,=\,c\,\frac{\partial\tau}{\partial
x^j}\quad\quad\forall\,j=0,1,2,3.
\end{equation}

 Next, as before, note that  the Kinematic pseudo-metric tensors of inertia $\{J^{ij}\}_{0\leq i,j\leq 3}$ and $\{J_{ij}\}_{0\leq i,j\leq 3}$  depend only on
the coordinate system in the space-time and are completely
independent of the physical matter or physical fields filling this
space. In contrast the pseudo-metric tensors of the four-dimensional
space-time $\{g^{ij}\}_{0\leq i,j\leq 3}$ and $\{g_{ij}\}_{0\leq
i,j\leq 3}$, given by \er{hoyuiouigyfghgjh3478zzrrZZffGGjkkj} and
\er{hoyuiouigyfg3478zzrrZZffGGhhjhj}, depend essentially on the
surrounding physical matter (in the model of the Newtonian gravity
through gravitational masses). Furthermore, note that in the absence
of the genuine gravity $\{g^{ij}\}_{0\leq i,j\leq 3}$ and
$\{g_{ij}\}_{0\leq i,j\leq 3}$ coincide with $\{J^{ij}\}_{0\leq
i,j\leq 3}$ and $\{J_{ij}\}_{0\leq i,j\leq 3}$ respectively. In the
model of the Newtonian gravity this happens away from essential
gravitational masses. In particular, $\{g^{ij}\}_{0\leq i,j\leq 3}$
and $\{g_{ij}\}_{0\leq i,j\leq 3}$ tend to $\{J^{ij}\}_{0\leq
i,j\leq 3}$ and $\{J_{ij}\}_{0\leq i,j\leq 3}$ respectively, as
$|\vec x|\to +\infty$. Finally note that since in every inertial
cartesian coordinate system $\vec k$ is a constant we deduce that in
every such a system $\{J^{ij}\}_{0\leq i,j\leq 3}$ and
$\{J_{ij}\}_{0\leq i,j\leq 3}$ are \underline{constant} in the
four-dimensional space-time. So, in contrast to $\{g^{ij}\}_{0\leq
i,j\leq 3}$ and $\{g_{ij}\}_{0\leq i,j\leq 3}$, the pseudo-metrics
$\{J^{ij}\}_{0\leq i,j\leq 3}$ and $\{J_{ij}\}_{0\leq i,j\leq 3}$
are complectly \underline{flat}. In particular, if we consider the
Christoffel Symbols of the Kinematic pseudo-metric tensors of
inertia $\{J_{ij}\}_{0\leq i,j\leq 3}$ defined by,
\begin{equation}\label{jhffhgffgh3478zzrrZZffjjkjkjklkkj88}
\begin{cases}
\left\{\Gamma_{i,mn}\right\}_J:=\frac{1}{2}\left(\frac{\partial
J_{im}}{\partial x_{n}}+\frac{\partial J_{in}}{\partial
x_{m}}-\frac{\partial
J_{mn}}{\partial x_{i}}\right)\\
\left\{\Gamma^{i}_{mn}\right\}_J:=\sum_{j=0}^{3}J^{ij}\left\{\Gamma_{j,mn}\right\}_J
\end{cases}\quad\quad\forall\, i,m,n=0,1,2,3,
\end{equation}
then, since in every inertial cartesian coordinate system
$\{J_{ij}\}_{0\leq i,j\leq 3}$ is \underline{constant} in the
four-dimensional space-time, we deduce that in every such coordinate
system we have
\begin{equation}\label{jhffhgffghhjgj3478zzrrZZffkhkhkljjllkjl88}
\left\{\Gamma^{i}_{mn}\right\}_J=\left\{\Gamma_{i,mn}\right\}_J=0\quad\quad\forall\,
i,m,n=0,1,2,3.
\end{equation}
Thus, if we consider the two times covariant tensor of the covariant
derivative to the covector of the gradient to the scalar of the
global time $\tau$ from \er{uguyytfddddgghkjhhjuuiui}, denoted by
$\left\{\delta_j \left(\frac{\partial\tau}{\partial x_i}
\right)\right\}_J$, and defined by:
\begin{multline}\label{hoyuiouigyfghgjh3478zzrrZZjhjhugyuyughggjhjhjyuuygtyffjjjkjljkjjljokjllllklljkljkjlnnnljjl88}
\left\{\delta_j \left(\frac{\partial\tau}{\partial x_i}
\right)\right\}_J=\left\{\delta_i \left(\frac{\partial\tau}{\partial
x_j} \right)\right\}_J:=\frac{\partial^2\tau}{\partial x_i\partial
x_j}-\sum_{m=0}^{3}\left\{\Gamma^{m}_{ij}\right\}_J\,\frac{\partial
\tau}{\partial x_m}
\\=\frac{\partial^2\tau}{\partial x_i\partial
x_j}-\sum_{m=0}^{3}\left\{\Gamma_{m,ij}\right\}_J\,
\frac{k^m}{c}\quad\quad\forall\, i,j=0,1,2,3\,,
\end{multline}
then by
\er{fgjfjhgghhgjghjhjkkkkgjghghuiiiulkkjlkklplikklkjljghhgghk} and
\er{jhffhgffghhjgj3478zzrrZZffkhkhkljjllkjl88} we prove the
following identity, first in every inertial cartesian coordinate
system and then, by the covariance of this identity, also in every
non-inertial cartesian coordinate system:
\begin{equation}\label{hoyuiouigyfghgjh3478zzrrZZjhjhugyuyughggjhjhjyuuygtyffjjjkjljkjjljokjllllklljkljkjlnnnhkhkhkhkhh88}
\left\{\delta_j \left(\frac{\partial\tau}{\partial x_i}
\right)\right\}_J=0\quad\quad\forall\, i,j=0,1,2,3\,.
\end{equation}
Note here that we put brackets $\left\{\cdot\right\}_J$ for the
Christoffel Symbols and covariant derivative with respect to
Kinematical pseudo-metrics  $\{J^{ij}\}_{0\leq i,j\leq 3}$ and
$\{J_{ij}\}_{0\leq i,j\leq 3}$ in order to distinguish them from
Christoffel Symbols and covariant derivative with respect to
Dynamical pseudo-metrics  $\{g^{ij}\}_{0\leq i,j\leq 3}$ and
$\{g_{ij}\}_{0\leq i,j\leq 3}$, that we will write
\underline{without} any brackets.

Furthermore, since $\vec k$ is \underline{generally trivial}
speed-like vector field, by the definition, there exists a unique,
up to equivalence, "preferable" inertial cartesian coordinate system
($\{0\}$) where we have
\begin{equation}\label{fgjfjhgghhgjghjhjijhojihjhjjijhjjjjjuiijjjkhjhjhjuiiuuuyu88jkhmbmbm88}
\vec k=0,
\end{equation}
and inserting it into \er{hoyuiouigyfghgjh3478zzrrZZffGGjkkj88} and
\er{hoyuiouigyfg3478zzrrZZffGGhhjhj88} we deduce that in system
($\{0\}$) we have
\begin{equation}\label{hoyuiouigyfghgjh3478zzrrZZffGGjkkj88jokhjh88}
\begin{cases}
J^{00}=1\\
J^{ij}=-\delta_{ij}\quad\forall 1\leq i,j\leq 3\\
J^{0j}=J^{j0}=0\quad\forall 1\leq j\leq 3,
\end{cases}
\end{equation}
and
\begin{equation}\label{hoyuiouigyfghgjh3478zzrrZZffGGjkkj88jokhjh88jkhkjhjh88}
\begin{cases}
J_{00}=1\\
J_{ij}=-\delta_{ij}\quad\forall 1\leq i,j\leq 3\\
J_{0j}=J_{j0}=0\quad\forall 1\leq j\leq 3,
\end{cases}
\end{equation}
which is the same as the basic pseudo-metrics in the Special
Relativity.

Next, we define the Dynamical four-covector of genuine gravity
$(s_0,s_1,s_2,s_3)$ by the following:
\begin{equation}\label{huohuioy89gjjhjffffff3478zzrrZZZhjhhjhhjjhhffGGhjjhuiiuihhjkkoioirrkhrrjggggjrrjiokuukuppkpklhjjhj88}
s_j=\frac{1}{2}\left(\sum_{m=0}^{3} g_{jm}k^m-\sum_{m=0}^{3}
J_{jm}v^m\right)\quad \forall j=0,1,2,3,
\end{equation}
where $(k^0,k^1,k^2,k^3)$ is the four-dimensional potential of
inertia, defined by
\er{fgjfjhgghhgjghjhjijhojihjhjjijhjjjjjuiijjjkjkkhhkkh},
$(v^0,v^1,v^2,v^3)$ is the four-dimensional gravitational potential,
defined by \er{fgjfjhgghhgjghjhjijhojihjhjjijhjjjjjuiijjjk},
$\{g_{ij}\}_{0\leq i,j\leq 3}$ defined by
\er{hoyuiouigyfg3478zzrrZZffGGhhjhj} and $\{J_{ij}\}_{0\leq i,j\leq
3}$ defined by \er{hoyuiouigyfg3478zzrrZZffGGhhjhj88}. Then, by
inserting \er{fgjfjhgghhgjghjhjijhojihjhjjijhjjjjjuiijjjkjkkhhkkh},
\er{fgjfjhgghhgjghjhjijhojihjhjjijhjjjjjuiijjjk},
\er{hoyuiouigyfg3478zzrrZZffGGhhjhj} and
\er{hoyuiouigyfg3478zzrrZZffGGhhjhj88} into
\er{huohuioy89gjjhjffffff3478zzrrZZZhjhhjhhjjhhffGGhjjhuiiuihhjkkoioirrkhrrjggggjrrjiokuukuppkpklhjjhj88}
we obtain
\begin{multline}\label{hklghgklghgklghjklghklgjklglhkklk88}
s_0=\frac{1}{2}\left(\sum_{m=0}^{3} g_{0m}k^m-\sum_{m=0}^{3}
J_{0m}v^m\right)=\frac{1}{2}\left(g_{00}k^0-
J_{00}v^0\right)+\frac{1}{2}\left(\sum_{m=1}^{3}
g_{0m}k^m-\sum_{m=1}^{3} J_{0m}v^m\right)\\=
\frac{1}{2}\left(\left(1-\frac{|\vec v|^2}{c^2}\right)-
\left(1-\frac{|\vec
k|^2}{c^2}\right)\right)+\frac{1}{2}\left(\sum_{m=1}^{3}
v^{m}k^m-\sum_{m=1}^{3} k^{m}v^m\right)=\frac{1}{2c^2}|\vec
k|^2-\frac{1}{2c^2}|\vec v|^2=\\-\frac{1}{2c^2}(\vec v-\vec
k)\cdot(\vec k+\vec v)=-\frac{1}{2c^2}|\vec v-\vec
k|^2-\frac{1}{c^2}\vec k\cdot(\vec v-\vec k)=\frac{1}{2c^2}|\vec
v-\vec k|^2-\frac{1}{c^2}\vec v\cdot(\vec v-\vec k),
\end{multline}
and
\begin{multline}\label{hklghgklghgklghjklghklgjklglhkklkjkguitguitnkh88}
s_j=\frac{1}{2}\left(\sum_{m=0}^{3} g_{jm}k^m-\sum_{m=0}^{3}
J_{jm}v^m\right)=\frac{1}{2}\left( g_{j0}k^0-
J_{j0}v^0\right)+\frac{1}{2}\left(\sum_{m=1}^{3}
g_{jm}k^m-\sum_{m=1}^{3} J_{jm}v^m\right)=\\
\frac{1}{2}\left( v^{j}k^0-
k^{j}v^0\right)+\frac{1}{2}\left(\sum_{m=1}^{3}
\left(-\delta_{jm}\right)k^m-\sum_{m=1}^{3}
\left(-\delta_{jm}\right)v^m\right)=v^j-k^j\quad \forall j=1,2,3.
\end{multline}
So, by \er{hklghgklghgklghjklghklgjklglhkklk88} and
\er{hklghgklghgklghjklghklgjklglhkklkjkguitguitnkh88} we deduce that
\begin{multline}\label{fgjfjhgghhgjghjhjijhojihjhjjijhjjjjjuiijjjk88}
(s_0,s_1,s_2,s_3)=\left(-\frac{1}{2c^2}|\vec h|^2-\frac{1}{c^2}\vec
k\cdot\vec h\,,\,\frac{1}{c}\vec h\right)=\left(\frac{1}{2c^2}|\vec
h|^2-\frac{1}{c^2}\vec v\cdot\vec h\,,\,\frac{1}{c}\vec
h\right)\quad\quad\text{where}\quad\\ s_0=-\frac{1}{2c^2}|\vec
h|^2-\frac{1}{c^2}\vec k\cdot\vec h=\frac{1}{2c^2}|\vec
h|^2-\frac{1}{c^2}\vec v\cdot\vec
h\;\;\quad\text{and}\quad\;\;(s_1,s_2,s_3)=\frac{1}{c}\vec h,
\end{multline}
where $\vec h$ is the three-dimensional proper vector field of the
vectorial potential of genuine gravity defined as in
\er{jljjjjkhhhjhjhjhljjjljljkjkCChhhh} by
\begin{equation}\label{huohuioy89gjjhjffffff3478zzrrZZZhjhhjhhjjhhffGGhjjhuiiuihhjkkoioirrkhrrjggggjrrjiokuukuppkpklhjjhjljoyhkiogu88}
\vec h:=\vec v-\vec k.
\end{equation}
In particular, by \er{hoyuiouigyfg3478zzrrZZffGGhhjhj},
\er{hoyuiouigyfg3478zzrrZZffGGhhjhj88},
\er{fgjfjhgghhgjghjhjijhojihjhjjijhjjjjjuiijjjk88} and
\er{fgjfjhgghhgjghjhjkkkkgjghghuiiiulkkjlkklplikklkjljghhgghk}
together we deduce:
\begin{equation}\label{hoyuiouigyfghgjh3478zzrrZZffGGjkkjojj88iihhkhk88jklhkhjkhk88}
{g}_{ij}=J_{ij}+\left(s_i\left(c\,\frac{\partial\tau}{\partial
x^j}\right)+\left(c\,\frac{\partial\tau}{\partial
x^i}\right)s_j\right)\quad\quad\forall\,i,j=0,1,2,3.
\end{equation}
Moreover, by \er{fgjfjhgghhgjghjhjijhojihjhjjijhjjjjjuiijjjk88}
\er{huohuioy89gjjhjffffff3478zzrrZZZhjhhjhhjjhhffGGhjjhuiiuihhjkkoioi}
and \er{hoyuiouigyfghgjh3478zzrrZZffGGjkkjojj88} we obtain:
\begin{equation}\label{hoyuiouigyfghgjh3478zzrrZZffGGjkkjojj88jklkhhk88khjhgh88}
v^j-k^j=\sum_{m=0}^{3}{\Theta}^{jm}s_m=\sum_{m=0}^{3}\left(k^jk^m-J^{jm}\right)s_m\quad\quad\forall\,j=0,1,2,3.
\end{equation}

\subsection{Physical laws in non-cartesian or curvilinear coordinate systems in the
non-relativistic space-time}
Let $\mathcal{S}$ be the group of all smooth non-degenerate
invertible transformations from $\mathbb{R}^4$ onto $\mathbb{R}^4$
having the form \er{fgjfjhgghyuyyu}:
\begin{equation}\label{fgjfjhgghyuyyuyughg}
\begin{cases}
x'^0=f^{(0)}(x^0,x^1,x^2,x^3),\\
x'^1=f^{(1)}(x^0,x^1,x^2,x^3),\\
x'^2=f^{(2)}(x^0,x^1,x^2,x^3),\\
x'^3=f^{(3)}(x^0,x^1,x^2,x^3),
\end{cases}
\end{equation}
and let $\mathcal{S}_0$ be a subgroup of transformations of the form
\er{noninchgravortbstrjgghguittu2intrrrZZygjyghhj}. Then, it is
clear, that given any object that is a scalar, four-vector,
four-covector, two-times covariant tensor or two-times contravariant
tensor on the group $\mathcal{S}_0$, defined in every cartesian
non-inertial coordinate system, we can uniquely extend the
definition of this object, in such a way that it will be defined
also in every curvilinear coordinate systems in $\mathbb{R}^4$ and
will be respectively a scalar, four-vector, four-covector, two-times
covariant tensor or two-times contravariant tensor on the wider
group $\mathcal{S}$. Thus all the physical laws that have a
covariant form preserve their form also in transformations of the
form \er{fgjfjhgghyuyyuyughg} i.e. in curvilinear coordinate
systems. In particular, the Maxwell Equations in every curvilinear
coordinate system have the form of
\er{khjhhkfgjfjhgghhgjghjhjkkkkgjghghuiiiulkkjlkklKKgfgjhjjghgjhhjhjhhhhhghhgtyytojjjjjuyiyuughgfg}
or equivalently of
\er{khjhhkfgjfjhgghhgjghjhjkkkkgjghghuiiiulkkjlkklKKgfgjhjjghgjhhjhjhhhhhghhgtyytojjjjjuyiyuughgfghhj}:
\begin{multline}\label{khjhhkfgjfjhgghhgjghjhjkkkkgjghghuiiiulkkjlkklKKgfgjhjjghgjhhjhjhhhhhghhgtyytojjjjjuyiyuughgfghjjhkpk}
\sum_{j=0}^{3}\frac{\partial}{\partial
x^j}\left(\sum_{m=0}^{3}\sum_{n=0}^{3}g^{km}g^{jn}\left(\frac{\partial
A_n}{\partial x^m}-\frac{\partial A_m}{\partial
x^n}\right)\right)+\\
\sum_{j=0}^{3}\frac{1}{\sqrt{|\text{det}\,G|}}\frac{\partial}{\partial
x^j}\left(\sqrt{|\text{det}\,G|}\right)\left(\sum_{m=0}^{3}\sum_{n=0}^{3}g^{km}g^{jn}\left(\frac{\partial
A_n}{\partial x^m}-\frac{\partial A_m}{\partial x^n}\right)\right)
=-4\pi j^k\quad\quad\forall\, k=0,1,2,3,
\end{multline}
or equivalently:
\begin{equation}\label{khjhhkfgjfjhgghhgjghjhjkkkkgjghghuiiiulkkjlkklKKgfgjhjjghgjhhjhjhhhhhghhgtyytojjjjjuyiyuughgfghhjkkhj}
\sum_{j=0}^{3}\frac{\partial}{\partial
x^j}\left(\sum_{m=0}^{3}\sum_{n=0}^{3}\sqrt{|\text{det}\,G|}\,g^{km}g^{jn}\left(\frac{\partial
A_n}{\partial x^m}-\frac{\partial A_m}{\partial x^n}\right)\right)
=-4\pi \sqrt{|\text{det}\,G|}\, j^k\quad\quad\forall\, k=0,1,2,3.
\end{equation}
Here $\{A_k\}_{k=0,1,2,3}$ is the four-covector of the
electromagnetic potential, $\{j^k\}_{k=0,1,2,3}$ is the four-vector
of the current and $G:=\{g_{kj}\}_{k,j=0,1,2,3}$,
$\{g^{kj}\}_{k,j=0,1,2,3}$ are pseudo-metric covariant and
contravariant tensors. Note, that in curvilinear coordinate system
we can have $\text{det}\,G\neq Const$ and thus we need to consider
the enlarged form
\er{khjhhkfgjfjhgghhgjghjhjkkkkgjghghuiiiulkkjlkklKKgfgjhjjghgjhhjhjhhhhhghhgtyytojjjjjuyiyuughgfg}
instead of
\er{khjhhkfgjfjhgghhgjghjhjkkkkgjghghuiiiulkkjlkklKKgfgjhjjghgjhhjhjhhhhhghhgtyytojjjjjuyiyuu}.
Moreover, the density of the Lagrangian of the electromagnetic field
in every curvilinear coordinate system in $\mathbb{R}^4$ also has a
form of \er{vhfffngghkjgghPPNggjgjjkgjoiuioiggh}:
\begin{multline}\label{vhfffngghkjgghPPNggjgjjkgjoiuioigghjhhh}
L_1=\frac{1}{4\pi}\left(-\sum_{n=0}^{3}\sum_{k=0}^{3}\frac{1}{4}F^{nk}F_{nk}-\sum_{k=0}^{3}4\pi
j^k A_k\right)=\\
\frac{1}{4\pi}\left(-\sum_{n=0}^{3}\sum_{k=0}^{3}\sum_{m=0}^{3}\sum_{p=0}^{3}\frac{1}{4}g^{mn}g^{pk}\left(\frac{\partial
A_p}{\partial x^m}-\frac{\partial A_m}{\partial
x^p}\right)\left(\frac{\partial A_k}{\partial x^n}-\frac{\partial
A_n}{\partial x^k}\right)-\sum_{k=0}^{3}4\pi j^k A_k\right),
\end{multline}
where
\begin{equation}\label{huohuioy89gjjhjffffff3478zzrrZZZhjhhjhhjjhhffGGhjjhjhhjhj}
F_{ij}:=\frac{\partial A_j}{\partial x^i}-\frac{\partial
A_i}{\partial x^j}\quad\quad\forall\, i,j=0,1,2,3\,.
\end{equation}

Next the general Lagrangian of motion of the charged particle in the
gravitational and electromagnetic field
\er{btfffygtgyggyijhhkkSYSPNuhuygyygyggyuyyuyuykuhghgjjoj} preserve
its form in every curvilinear coordinate system:
\begin{equation}\label{btfffygtgyggyijhhkkSYSPNuhuygyygyggyuyyuyuykuhghgjjojhjfg}
J_{\mathcal{G}}(\chi)=
\int_0^T\left\{-mc^2\;\mathcal{G}\left(\sum_{j=0}^{3}\sum_{k=0}^{3}g_{jk}\left(\chi(t)\right)\,\frac{d\chi^j}{dt}\,\frac{d\chi^k}{dt}\right)-\sum_{k=0}^{3}\sigma
A_k\left(\chi(t)\right)\,\frac{d\chi^k}{dt}\right\}dt.
\end{equation}
where $t$ is the global time, which is a scalar on the group
$\mathcal{S}$,
\begin{equation}\label{btfffjhgjghghijhh}
\left(\chi^0(t),\chi^1(t),\chi^2(t),\chi^3(t)\right):=\left(\frac{1}{c}x^0(t),\frac{1}{c}x_1(t),\frac{1}{c}x_2(t),\frac{1}{c}x_3(t)\right),
\end{equation}
and $\left(x^0(t),x^1(t),x^2(t),x^3(t)\right)\in\mathbb{R}^4$ is a
four-dimensional space-time trajectory of the particle,
parameterized by the global time.

Note that if we denote by $t$ the scalar of global time, then in a
general curvilinear coordinate system the coordinate $x^0$ can
differ from $ct$, and the equality $x^0=ct$ valid, in general,
\underline{only} in \underline{cartesian} inertial or non-inertial
coordinate systems. However, since the equality in
\er{fgjfjhgghhgjghjhjijhojihjhjjijhjjjjjuiijjjkhjhjhjuiiuuuyu} has a
covariant form, the scalar of the global time $t$ satisfies the
following Eikonal-type equation in every curvilinear coordinate
system:
\begin{equation}\label{fgjfjhgghhgjghjhjijhojihjhjjijhjjjjjuiijjjkhjhjhjuiiuuuyuuoiuuiio}
\sum_{j=0}^{3}\sum_{m=0}^{3}\,g^{jm}(x^0,x^1,x^2,x^3)\,\frac{\partial
t}{\partial x^j}(x^0,x^1,x^2,x^3)\,\frac{\partial t}{\partial
x^m}(x^0,x^1,x^2,x^3)\,=\, \frac{1}{c^2}.
\end{equation}
Moreover, since the equality in
\er{fgjfjhgghhgjghjhjijhojihjhjjijhjjjjjuiijjjkhjhjhjuiiuuuyu} also
has a covariant form, the following identity is valid in every
curvilinear coordinate system:
\begin{equation}\label{fgjfjhfgfgfgfhihhjhhjjhhjhj}
\sum_{m=0}^{3}{\Theta}^{jm}\frac{\partial t}{\partial
x^m}=0\quad\quad\forall\, j=0,1,2,3,
\end{equation}
where $\Theta^{ij}$ is the contravariant tensor of the
three-dimensional geometry, that has the form
\er{huohuioy89gjjhjffffff3478zzrrZZZhjhhjhhjjhhffGGhjjhuiiuihhjkkoioi}
\underline{only} in \underline{cartesian} inertial or non-inertial
coordinate systems.

Next, in the particular case of the relativistic-like Lagrangian
where $\mathcal{G}(\tau):=\sqrt{\tau}$, the Lagrangian in
\er{btfffygtgyggyijhhkkSYSPNuhuygyygyggyuyyuyuykuhghgjjojiyyjljlgh}
also preserve their form in every curvilinear coordinate system:
\begin{equation}\label{btfffygtgyggyijhhkkSYSPNuhuygyygyggyuyyuyuykuhghgjjojiyyjljlghhhhjhj}
J_{rl}(\chi)=
\int_a^b\left\{-mc^2\;\sqrt{\left(\sum_{j=0}^{3}\sum_{k=0}^{3}g_{jk}\left(\chi(s)\right)\,\frac{d\chi^j}{ds}\,\frac{d\chi^k}{ds}\right)}-\sum_{k=0}^{3}\sigma
A_k\left(\chi(s)\right)\,\frac{d\chi^k}{ds}\right\}ds,
\end{equation}
where $s$ is the arbitrary parameter of the trajectory with
increasing mapping $t\leftrightarrow s$:
\begin{equation}\label{btfffjhgjghghijhhhu}
\left(\chi^0(s),\chi^1(s),\chi^2(s),\chi^3(s)\right):=\left(\frac{1}{c}x^0(s),\frac{1}{c}x_1(s),\frac{1}{c}x_2(s),\frac{1}{c}x_3(s)\right).
\end{equation}
In particular, in the case $\frac{\partial \chi^0}{\partial t}>0$ we
can take $s:=\chi^0$ in
\er{btfffygtgyggyijhhkkSYSPNuhuygyygyggyuyyuyuykuhghgjjojiyyjljlghhhhjhj}.

Finally, we would like to note the following fact: since in the
absence of essential gravitational masses we have $g_{ij}=J_{ij}$,
there exists a unique, up to equivalence, "preferable" inertial
coordinate system where $\vec v=\vec k=0$ everywhere. In this
particular system by
\er{hoyuiouigyfghgjh3478zzrrZZffGGjkkj88jokhjh88jkhkjhjh88} we have:
\begin{equation}\label{hoyuiouigyfg3478zzrrZZffGGhhjhjiui}
\begin{cases}
g_{00}
=1\\
g_{ij}
=-\delta_{ij}\quad\forall 1\leq i,j\leq 3\\
g_{0j}=g_{j0}
=0\quad\forall 1\leq j\leq 3.
\end{cases}
\end{equation}
and thus the Maxwell equations are the same as in the Special
Relativity. Moreover, in this system the Lagrangian of the motion of
the particle of the form
\er{btfffygtgyggyijhhkkSYSPNuhuygyygyggyuyyuyuykuhghgjjojiyyjljlghhhhjhj}
is also the same as in the Special Relativity. Thus, since Maxwell
equations
\er{khjhhkfgjfjhgghhgjghjhjkkkkgjghghuiiiulkkjlkklKKgfgjhjjghgjhhjhjhhhhhghhgtyytojjjjjuyiyuughgfghjjhkpk}
and the Lagrangian of the motion of particles
\er{btfffygtgyggyijhhkkSYSPNuhuygyygyggyuyyuyuykuhghgjjojiyyjljlghhhhjhj}
preserve their form in every cartesian, non-cartesian or curvilinear
coordinate system of the group $\mathcal{S}$, they stay the same as
in Special Relativity also in the case of every cartesian,
non-cartesian or curvilinear coordinate system. Thus in the
particular case of $\mathcal{G}(\tau):=\sqrt{\tau}$ in
\er{btfffygtgyggyijhhkkSYSPNuhuygyygyggyuyyuyuykuhghgjjojhjfg} and
in the absence of essential gravitational masses, the unique formal
mathematical difference between our model and the Special Relativity
is that in the frames of our model we consider the Galilean
Transformations as transformations of the change of inertial
cartesian coordinate systems, up to equivalence, and
\er{noninchgravortbstr} as transformations of the change of
non-inertial cartesian coordinate systems, however the Lorenz
transformations lead to inertial \underline{non-cartesian}
coordinate systems (see the following Definition
\ref{gjgjkfgkjfjkfjfjdf}). In contrast, in the Special Relativity
the fundamental role of the Lorenz transformations, i.e. the
transformations that preserve the form
\er{hoyuiouigyfg3478zzrrZZffGGhhjhjiui} of the pseudo-metric tensor,
is postulated as the role of transformations of the change of
inertial cartesian coordinate systems, and at the same time the
Galilean Transformations lead to inertial \underline{non-cartesian}
coordinate systems, and transformations \er{noninchgravortbstr} lead
to non-inertial \underline{non-cartesian} coordinate systems.

\subsubsection{Some general covariant identities in non-cartesian or curvilinear coordinate systems}
Consider the contravariant kinematic pseudo-metric tensor of inertia
$\{J^{ij}\}_{0\leq i,j\leq 3}$ on the group $\mathcal{S}$, defined
by
\begin{equation}\label{hoyuiouigyfghgjh3478zzrrZZffGGjkkjojj88jkhkhjh}
J^{ij}:=k^ik^j-{\Theta}^{ij}\quad\quad\forall\,i,j=0,1,2,3,
\end{equation}
where $\{{\Theta}^{ij}\}_{0\leq i,j\leq 3}$ is the contravariant
tensor of the three-dimensional geometry on the group $\mathcal{S}$
(the same as in \er{fgjfjhfgfgfgfhihhjhhjjhhjhj}) defined in
cartesian coordinate systems by
\er{huohuioy89gjjhjffffff3478zzrrZZZhjhhjhhjjhhffGGhjjhuiiuihhjkkoioi},
and $(k^0,k^1,k^2,k^3)$ is the four-vector of the potential of
inertia on the group $\mathcal{S}$, defined in cartesian coordinate
systems by \er{fgjfjhgghhgjghjhjijhojihjhjjijhjjjjjuiijjjkjkkhhkkh}.
Then, as before, the reverse to $\{J^{ij}\}_{0\leq i,j\leq 3}$
covariant tensor on the group $\mathcal{S}$ is denoted as
$\{J_{ij}\}_{0\leq i,j\leq 3}$ and satisfies
\begin{equation}\label{fgjfjhgghhgjghjhjkkkkgjghghuiiiuujhjh88jgkgjggk}
\sum_{m=0}^{3}J^{im}J_{mj}=\begin{cases}
1\quad\text{if}\quad i=j\\
0\quad\text{if}\quad i\neq j
\end{cases}\quad\quad\forall\, i,j=0,1,2,3
\end{equation}
(as in \er{fgjfjhgghhgjghjhjkkkkgjghghuiiiuujhjh88}).
\begin{definition}\label{gjgjkfgkjfjkfjfjdfjhh} We say
that a given general coordinate system ($*$) is
\underline{cartesian} if the contravariant tensor of the
three-dimensional geometry $\{{\Theta}^{ij}\}_{0\leq i,j\leq 3}$ has
the following simple form in the system ($*$) (the same as in
\er{huohuioy89gjjhjffffff3478zzrrZZZhjhhjhhjjhhffGGhjjhuiiuihhjkkoioi}):
\begin{equation}\label{huohuioy89gjjhjffffff3478zzrrZZZhjhhjhhjjhhffGGhjjhuiiuihhjkkoioijkljljkjkjk}
\begin{cases}
{\Theta}^{00}=0
\\
{\Theta}^{0j}={\Theta}^{j0}=0\quad\quad\forall\, j=1,2,3
\\
{\Theta}^{ij}:=\delta_{ij}\quad\quad\forall\, i,j=1,2,3,
\end{cases}
\end{equation}
and at the same time in the system ($*$) we have
\begin{equation}\label{fgjfjhgghhgjghjhjijhojihjhjjijhjjjjjuiijjjkhjhjhjuiiuuuyuuoiuuiiokljkjkjljlj}
c\frac{\partial t}{\partial x^0}(x^0,x^1,x^2,x^3)=1,
\end{equation}
where $t$ is the scalar of global time (the same as in
\er{fgjfjhfgfgfgfhihhjhhjjhhjhj}). In particular, by
\er{fgjfjhfgfgfgfhihhjhhjjhhjhj},
\er{huohuioy89gjjhjffffff3478zzrrZZZhjhhjhhjjhhffGGhjjhuiiuihhjkkoioijkljljkjkjk}
and
\er{fgjfjhgghhgjghjhjijhojihjhjjijhjjjjjuiijjjkhjhjhjuiiuuuyuuoiuuiiokljkjkjljlj}
all together in the system ($*$) we have:
\begin{equation}\label{fgjfjhgghhgjghjhjijhojihjhjjijhjjjjjuiijjjkhjhjhjuiiuuuyuuoiuuiiokljkjkjljjkkhhl}
c\left(\frac{\partial t}{\partial
x^0}(x^0,x^1,x^2,x^3),\frac{\partial t}{\partial
x^1}(x^0,x^1,x^2,x^3),\frac{\partial t}{\partial
x^2}(x^0,x^1,x^2,x^3),\frac{\partial t}{\partial
x^3}(x^0,x^1,x^2,x^3)\right)=(1,0,0,0).
\end{equation}
Then, it easily can be shown that given an cartesian coordinate
system ($*$) and a general coordinate system ($**$), the system
($**$) is cartesian if and only if, up to a constant shift of the
time, the change of coordinates from system ($*$) to system ($**$)
is given by \er{noninchgravortbstrjgghguittu2intrrrZZygjyghhj} or
equivalently by \er{noninchgravortbstrjgghguittu2intrrrZZygjyg}:
\begin{equation}\label{noninchgravortbstrjgghguittu2intrrrZZygjyghkujhkhilyiljg}
\begin{cases}
\vec x'=A\left(\frac{x_0}{c}\right)\cdot\vec x+\vec z\left(\frac{x_0}{c}\right)=A\left(t\right)\cdot\vec x+\vec z\left(t\right),\\
x'_0=x_0=t,
\end{cases}
\end{equation}
where $A(t)\in SO(3)$ is a rotation.
\end{definition}
\begin{definition}\label{gjgjkfgkjfjkfjfjdf} Similarly to cartesian coordinate systems, we say
that a given non-cartesian or curvilinear coordinate system ($*$) is
\underline{inertial} if the \underline{four-vector} of the potential
of inertia $(k^0,k^1,k^2,k^3)$, defined in cartesian systems by
\er{fgjfjhgghhgjghjhjijhojihjhjjijhjjjjjuiijjjkjkkhhkkh} is constant
in $\mathbb{R}^4$. Then, it easily can be shown that given an
inertial coordinate system ($*$) and a general coordinate system
($**$), the system ($**$) is inertial if and only if the change of
coordinates from system ($*$) to system ($**$) is given by a
\underline{linaer} transformation in $\mathbb{R}^4$. In particular,
a general coordinate system ($***$) is inertial if and only if the
tensors $\{J^{ij}\}_{0\leq i,j\leq 3}$ and $\{J_{ij}\}_{0\leq
i,j\leq 3}$ are \underline{constant} in the four-dimensional
space-time in the given system ($***$). Moreover, given an inertial
cartesian coordinate system ($*$) and a general coordinate system
($**$), the system ($**$) is inertial and at the same time is
cartesian if and only if, up to a constant shift of time and, up to
equivalence of coordinate systems, the change of coordinates from
system ($*$) to system ($**$) is given by the Galilean
transformation.
\end{definition}
In the general non-cartesian or curvilinear coordinate system we
obviously have the following list of general covariant identities:
\begin{itemize}
\item
\begin{equation}\label{fgjfjhgghhgjghjhjkkkkgjghghuiiiuujhjhljhgugii88;jl;jk}
\text{det}\,G=\text{det}\,J
\end{equation}
(see \er{fgjfjhgghhgjghjhjkkkkgjghghuiiiuujhjhljhgugii88}), where
$4\times 4$-matrix $J$ is given by,
\begin{equation}\label{fgjfjhgghhgjghjhjkkkkgjghghuiiiuujhjhjkljojuuiyi88gkgkkkkhhkhkh}
J=\{J_{ij}\}_{0\leq i,j\leq 3},
\end{equation}
and $4\times 4$-matrix $G$ is given by,
\begin{equation}\label{fgjfjhgghhgjghjhjkkkkgjghghuiiiuujhjhjkljojuuiyi88gkgkk;}
G=\{g_{ij}\}_{0\leq i,j\leq 3},
\end{equation}
with the covariant pseudo-metric tensor of the four-dimensional
space-time $\{g_{ij}\}_{0\leq i,j\leq 3}$.
\item
\begin{equation}\label{hoyuiouigyfghgjh3478zzrrZZjhjhugyuyughggjhjhjyuuygtyffjjjkjljkjjljokjllllklljkljkjlnnnhkhkhkhkhh8899}
\left\{\delta_j \left(\frac{\partial t}{\partial x_i}
\right)\right\}_J=0\quad\quad\forall\, i,j=0,1,2,3
\end{equation}
(see
\er{hoyuiouigyfghgjh3478zzrrZZjhjhugyuyughggjhjhjyuuygtyffjjjkjljkjjljokjllllklljkljkjlnnnhkhkhkhkhh88}),
where
\begin{multline}\label{hoyuiouigyfghgjh3478zzrrZZjhjhugyuyughggjhjhjyuuygtyffjjjkjljkjjljokjllllklljkljkjlnnnljjl8899}
\left\{\delta_j \left(\frac{\partial t}{\partial x_i}
\right)\right\}_J=\left\{\delta_i \left(\frac{\partial t}{\partial
x_j} \right)\right\}_J:=\frac{\partial^2 t}{\partial x^i\partial
x^j} -\sum_{m=0}^{3}\left\{\Gamma^{m}_{ij}\right\}_J\,\frac{\partial
t}{\partial x^m}\quad\quad\forall\, i,j=0,1,2,3\,,
\end{multline}
with
\begin{equation}\label{jhffhgffgh3478zzrrZZffjjkjkjklkkj8899}
\begin{cases}
\left\{\Gamma_{i,mn}\right\}_J:=\frac{1}{2}\left(\frac{\partial
J_{im}}{\partial x_{n}}+\frac{\partial J_{in}}{\partial
x_{m}}-\frac{\partial
J_{mn}}{\partial x_{i}}\right)\\
\left\{\Gamma^{i}_{mn}\right\}_J:=\sum_{j=0}^{3}J^{ij}\left\{\Gamma_{j,mn}\right\}_J
\end{cases}\quad\quad\forall\, i,m,n=0,1,2,3.
\end{equation}
\item
\begin{equation}\label{fgjfjhgghhgjghjhjijhojihjhjjijhjjjjjuiijjjkhjhjhjuiiuuuyu88jukukukuk88kpjljl99}
\sum_{m=0}^{3} \,J^{jm}\,\left(c\,\frac{\partial t}{\partial
x^m}\right)=k^j\quad\quad\forall\,j=0,1,2,3,
\end{equation}
and
\begin{equation}\label{fgjfjhgghhgjghjhjijhojihjhjjijhjjjjjuiijjjkhjhjhjuiiuuuyu88jukukukuk88kpjljlkk99}
\sum_{m=0}^{3} \,g^{jm}\,\left(c\,\frac{\partial t}{\partial
x^m}\right)=v^j\quad\quad\forall\,j=0,1,2,3.
\end{equation}
\item
\begin{equation}\label{fgjfjhgghhgjghjhjijhojihjhjjijhjjjjjuiijjjkhjhjhjuiiuuuyu88jukukukuk88kpjljlkjojjhgjkl}
\sum_{j=0}^{3} \,k^j\,\left(c\,\frac{\partial t}{\partial
x^j}\right)=1=\sum_{j=0}^{3} \,v^j\,\left(c\,\frac{\partial
t}{\partial x^j}\right),
\end{equation}
\item
\begin{equation}\label{fgjfjhgghhgjghjhjijhojihjhjjijhjjjjjuiijjjkhjhjhjuiiuuuyu88mnmbbbjjj}
c^2\left(\sum_{j=0}^{3}\sum_{m=0}^{3}\,J^{jm}\frac{\partial
t}{\partial x^j}\frac{\partial t}{\partial
x^m}\right)=1=c^2\left(\sum_{j=0}^{3}\sum_{m=0}^{3}\,g^{jm}\frac{\partial
t}{\partial x^j}\frac{\partial t}{\partial x^m}\right).
\end{equation}
\item
\begin{equation}\label{fgjfjhgghhgjghjhjijhojihjhjjijhjjjjjuiijjjkhjhjhjuiiuuuyu88jukukukuk88kpjljlmhmbbvnvvvkljljl88jkhkhk}
\sum_{m=0}^{3} \,J_{jm}\,k^m\,=\,c\,\frac{\partial t}{\partial
x^j}\,=\,\sum_{m=0}^{3} \,g_{jm}\,v^m\quad\quad\forall\,j=0,1,2,3.
\end{equation}
\item
\begin{equation}\label{fgjfjhfgfgfgfhihhjhhjjhhjhjljjljkljjk}
\sum_{m=0}^{3}{\Theta}^{jm}\frac{\partial t}{\partial
x^m}=0\quad\quad\forall\, j=0,1,2,3.
\end{equation}
\item
\begin{equation}\label{hoyuiouigyfghgjh3478zzrrZZffGGjkkjojj88jkhkhjhjljljklkl}
J^{ij}=k^ik^j-{\Theta}^{ij}\quad\quad\forall\,i,j=0,1,2,3,
\end{equation}
and
\begin{equation}\label{hoyuiouigyfghgjh3478zzrrZZffGGjkkjojj88jkhkhjhjljljjjkhhkhjk;k;}
g^{ij}=v^iv^j-{\Theta}^{ij}\quad\quad\forall\,i,j=0,1,2,3,
\end{equation}
that implies, in particular,
\begin{equation}\label{hoyuiouigyfghgjh3478zzrrZZffGGjkkjojj88jkhkhjhjljljjjkhhkhjk;k;jkljhk}
g^{ij}=\left(v^iv^j-k^ik^j\right)+J^{ij}\quad\quad\forall\,i,j=0,1,2,3.
\end{equation}
\item
\begin{equation}\label{fgjfjhgghhgjghjhjkkkkgjghghuiiiuujhjh88jgkgjggkjjjk}
\sum_{m=0}^{3}J^{im}J_{mj}=\sum_{m=0}^{3}g^{im}g_{mj}=\begin{cases}
1\quad\text{if}\quad i=j\\
0\quad\text{if}\quad i\neq j
\end{cases}\quad\quad\forall\, i,j=0,1,2,3,
\end{equation}
that implies together with
\er{fgjfjhgghhgjghjhjijhojihjhjjijhjjjjjuiijjjkhjhjhjuiiuuuyu88jukukukuk88kpjljlmhmbbvnvvvkljljl88jkhkhk}
and
\er{fgjfjhgghhgjghjhjijhojihjhjjijhjjjjjuiijjjkhjhjhjuiiuuuyu88mnmbbbjjj}
the following:
\begin{equation}\label{fgjfjhgghhgjghjhjijhojihjhjjijhjjjjjuiijjjkhjhjhjuiiuuuyu88mnmbbbjjjoouuoojou}
\sum_{j=0}^{3}\sum_{m=0}^{3}\,J_{jm}
k^jk^m=1=\sum_{j=0}^{3}\sum_{m=0}^{3}\,g_{jm}v^jv^m.
\end{equation}
\end{itemize}
Moreover, in the general non-cartesian or curvilinear coordinate
system, as before, we can define the Dynamical four-covector of
genuine gravity $(s_0,s_1,s_2,s_3)$ by:
\begin{equation}\label{huohuioy89gjjhjffffff3478zzrrZZZhjhhjhhjjhhffGGhjjhuiiuihhjkkoioirrkhrrjggggjrrjiokuukuppkpklhjjhj88khhjh}
s_j=\frac{1}{2}\left(\sum_{m=0}^{3} g_{jm}k^m-\sum_{m=0}^{3}
J_{jm}v^m\right)\quad \forall j=0,1,2,3
\end{equation}
(see
\er{huohuioy89gjjhjffffff3478zzrrZZZhjhhjhhjjhhffGGhjjhuiiuihhjkkoioirrkhrrjggggjrrjiokuukuppkpklhjjhj88}).
Then we have the following covariant identities
\begin{equation}\label{hoyuiouigyfghgjh3478zzrrZZffGGjkkjojj88iihhkhk88jklhkhjkhk8899}
{g}_{ij}=J_{ij}+\left(s_i\left(c\,\frac{\partial t}{\partial
x^j}\right)+\left(c\,\frac{\partial t}{\partial
x^i}\right)s_j\right)\quad\quad\forall\,i,j=0,1,2,3
\end{equation}
(see
\er{hoyuiouigyfghgjh3478zzrrZZffGGjkkjojj88iihhkhk88jklhkhjkhk88}),
and
\begin{equation}\label{hoyuiouigyfghgjh3478zzrrZZffGGjkkjojj88jklkhhk88khjhgh8899}
v^j-k^j=\sum_{m=0}^{3}{\Theta}^{jm}s_m=\sum_{m=0}^{3}\left(k^jk^m-J^{jm}\right)s_m\quad\quad\forall\,j=0,1,2,3
\end{equation}
(see \er{hoyuiouigyfghgjh3478zzrrZZffGGjkkjojj88jklkhhk88khjhgh88}).

\begin{definition}\label{gjkfgjhfgdhdgdh}
We say that a given general coordinate system ($*$) is
\underline{Lorentzian} if in system ($*$) kinematic pseudo-metric
tensors of inertia $\{J^{ij}\}_{0\leq i,j\leq 3}$ and
$\{J_{ij}\}_{0\leq i,j\leq 3}$ have the following simple form (the
same as in \er{hoyuiouigyfghgjh3478zzrrZZffGGjkkj88jokhjh88} and
\er{hoyuiouigyfghgjh3478zzrrZZffGGjkkj88jokhjh88jkhkjhjh88}):
\begin{equation}\label{hoyuiouigyfghgjh3478zzrrZZffGGjkkj88jokhjh8899}
\begin{cases}
J^{00}=1\\
J^{ij}=-\delta_{ij}\quad\forall 1\leq i,j\leq 3\\
J^{0j}=J^{j0}=0\quad\forall 1\leq j\leq 3,
\end{cases}
\end{equation}
and
\begin{equation}\label{hoyuiouigyfghgjh3478zzrrZZffGGjkkj88jokhjh88jkhkjhjh8899}
\begin{cases}
J_{00}=1\\
J_{ij}=-\delta_{ij}\quad\forall 1\leq i,j\leq 3\\
J_{0j}=J_{j0}=0\quad\forall 1\leq j\leq 3.
\end{cases}
\end{equation}
Thus, since in every Lorentzian coordinate system $\{J^{ij}\}_{0\leq
i,j\leq 3}$ and $\{J_{ij}\}_{0\leq i,j\leq 3}$ are constant in
$\mathbb{R}^4$, as it was explained in the end of Definition
\ref{gjgjkfgkjfjkfjfjdf} every Lorentzian coordinate system is
necessary \underline{inertial}, however it is not necessary
cartesian. In particular, a unique, up to equivalence, "preferable"
inertial cartesian coordinate system ($\{0\}$) where the
three-dimensional vector potential of inertia satisfies $\vec k=0$,
is a unique, up to equivalence, coordinate system which is cartesian
and Lorentzian \underline{simultaneously}. On the other hand, there
exists a plenty of Lorentzian coordinate systems. In particular, it
is well known that, given a Lorentzian coordinate system ($*$) and a
general coordinate system ($**$), the system ($**$) is also
Lorentzian, if and only if the change of coordinates from system
($*$) to system ($**$) is given by the usual \underline{Lorentz}
transformation, up to certain trivial transformations. In every
Lorentzian coordinate system the four-vector of the potential of
inertia $(k^1,k^2,k^3,k^4)$ is \underline{constant} which is by
\er{fgjfjhgghhgjghjhjijhojihjhjjijhjjjjjuiijjjkhjhjhjuiiuuuyu88mnmbbbjjjoouuoojou}
and \er{hoyuiouigyfghgjh3478zzrrZZffGGjkkj88jokhjh88jkhkjhjh8899}
always satisfies
\begin{equation}\label{hoyuiouigyfghgjh3478zzrrZZffGGjkkjojj88iihhkhk88jklhkhjkhk8899jhjhjh}
\left(k^0\right)^2-\left(\left(k^1\right)^2+\left(k^2\right)^2+\left(k^3\right)^2\right)=1.
\end{equation}
In fact in a general Lorentzian coordinate system the four-vector
$(k^1,k^2,k^3,k^4)$ can be arbitrary constant four-vector,
satisfying
\er{hoyuiouigyfghgjh3478zzrrZZffGGjkkjojj88iihhkhk88jklhkhjkhk8899jhjhjh}.
We remind that in the unique cartesian Lorentzian coordinate system
we have $(k^1,k^2,k^3,k^4)=(1,0,0,0)$. Next by
\er{fgjfjhgghhgjghjhjijhojihjhjjijhjjjjjuiijjjkhjhjhjuiiuuuyu88jukukukuk88kpjljlmhmbbvnvvvkljljl88jkhkhk}
together with \er{hoyuiouigyfghgjh3478zzrrZZffGGjkkj88jokhjh8899} we
have:
\begin{equation}\label{fgjfjhgghhgjghjhjijhojihjhjjijhjjjjjuiijjjkhjhjhjuiiuuuyu88jukukukuk88kpjljlmhmbbvnvvvkljljl88jkhkhkjkjkjkjkklhhhjhkhhhk}
\left(\frac{\partial t}{\partial x^0},\frac{\partial t}{\partial
x^1},\frac{\partial t}{\partial x^2},\frac{\partial t}{\partial
x^3}\right) \,=\, \frac{1}{c}\left(k^0,-k^1,-k^2,-k^3\right),
\end{equation}
and so, up to a constant shift in time, in every Lorentzian
coordinate system we have
\begin{equation}\label{hoyuiouigyfghgjh3478zzrrZZffGGjkkjojj88iihhkhk88jklhkhjkhk8899jhjhjhkjhkjjjk}
t\,=\,\frac{1}{c}\,k^0x^0-\frac{1}{c}\left(k^1x^1+k^2x^2+k^3x^3\right).
\end{equation}
Note that in the unique cartesian Lorentzian coordinate system we
have $t=\frac{x^0}{c}$.
\end{definition}

\subsection{Certain curvilinear coordinate system in the case of
stationary radially symmetric gravitational field and relation to
the Schwarzschild metric}\label{ghghfhffgfggffgfg} Assume that for a
given part of the space in some inertial or non-inertial cartesian
coordinate system $(*)$ the gravitational field is stationary and
radially symmetric that means that the vectorial gravitational
potential $\vec v=(v_1,v_2,v_3)$ is independent of time variable $t$
and having the form
\begin{equation}\label{ggyfgdfdds}
\vec v(\vec x)=g\left(|\vec x|\right)\,\frac{\vec x}{|\vec
x|}\quad\quad\forall\,\vec x,
\end{equation}
for some scalar function $g(s):\mathbb{R}\to\mathbb{R}$. Next, given
some differentiable function $\Theta(\vec
x):\mathbb{R}^3\to\mathbb{R}$, consider the change of variables in
the four-dimensional space-time $\mathbb{R}^4$:
\begin{equation}\label{giuuihjghgghjgj78zzrrZZffhhhggygghghhjjknnKKMMghgh}
\begin{cases}
x'^0=x^0+\frac{\Theta\left(\left(x^1,x^2,x^3\right)\right)}{c}\\
x'^j=x^j\quad\quad\forall j=1,2,3.
\end{cases}
\end{equation}
that transforms the cartesian coordinate system $(*)$ to the
curvilinear coordinate system $(**)$ in the four-dimensional
space-time $\mathbb{R}^4$. Then in the terms of the
three-dimensional space and one-dimensional time:
\begin{equation}\label{fgjfjhgghjhjhjjhhjhh}
(x^0,x^1,x^2,x^3):=(ct,x_1,x_2,x_3)=\left(ct,\vec x\right),
\end{equation}
we rewrite \er{giuuihjghgghjgj78zzrrZZffhhhggygghghhjjknnKKMMghgh}
as:
\begin{equation}\label{giuuihjghgghjgj78zzrrZZffhhhggygghghhjjknnKKMM}
\begin{cases}
t'=t+\frac{\Theta(\vec x)}{c^2}\\
\vec x'=\vec x.
\end{cases}
\end{equation}
Note again, that since the new coordinate system $(**)$ in
$\mathbb{R}^4$ is \underline{curvilinear}, the time-like coordinate
$t'$ in coordinate system $(**)$ \underline{differ} from the proper
scalar of the global time.

Next if we define a matrix
\begin{equation}\label{giuuihjghg78zzrrZZffhhhgguyuyynnKKMM}
A=\left\{a^{i}_{j}\right\}_{0\leq i,j\leq 3}\in\mathbb{R}^{4\times
4}=\left\{\frac{\partial x'^i}{\partial x^j}\right\}_{0\leq i,j\leq
3}\in\mathbb{R}^{4\times 4},
\end{equation}
then
\begin{equation}\label{hoyuiouigyfg3478zzrrZZffhhhggjhhnnKKMM}
\begin{cases}
a^{0}_{0}=1
\\
a^{i}_{j}=\delta_{ij}\quad\forall 1\leq i,j\leq 3\\
a^{0}_{j}=\frac{1}{c}\frac{\partial \Theta}{\partial x^j}\quad\forall 1\leq j\leq 3\\
a^{j}_{0}=0\quad\forall 1\leq j\leq 3.
\end{cases}
\end{equation}
Next consider the contravariant pseudo-metric tensor of the
four-dimensional space-time $\{g^{ij}\}_{0\leq i,j\leq 3}$ that due
to \er{hoyuiouigyfghgjh3478zzrrZZffGGjkkj} has the form of
\begin{equation}\label{hoyuiouigyfg3478zzrrZZffhhhggughjttnnKKMM}
\begin{cases}
g^{00}=1\\
g^{ij}=-\delta_{ij}+\frac{v^iv^j}{c^2}\quad\forall 1\leq i,j\leq 3\\
g^{0j}=g^{j0}=\frac{v^j}{c}\quad\forall 1\leq j\leq 3,
\end{cases}
\end{equation}
in the cartesian coordinate system $(*)$. We would like to find the
form $\{g'^{ij}\}_{0\leq i,j\leq 3}$ of this tensor in the
curvilinear coordinate system $(**)$. Then by
\er{fgjfjhgghhgjghjhjkkkkgggh} we have:
\begin{equation}\label{giuuihjghg78zzrrZZffhhhgguyuyyolhjhjhjjjkttnnKKMM}
g'^{mn}=\sum_{i=0}^{3}\sum_{j=0}^{3}\left(a_{i}^{m}g^{ij}a_{j}^{n}\right)=\sum_{i=0}^{3}a_{i}^{m}\left(\sum_{j=0}^{3}g^{ij}a_{j}^{n}\right)\quad\quad\forall
0\leq m,n\leq 3.
\end{equation}
I.e.
\begin{equation}\label{giuuihjghg78zzrrZZffhhhgguyuyyolhjhjhjjjkjjjjttnnKKMM}
g'^{mn}=a_{0}^{m}\left(g^{00}a_{0}^{n}+\sum_{j=1}^{3}w^{0j}a_{j}^{n}\right)+\sum_{i=1}^{3}a_{i}^{m}\left(g^{i0}a_{0}^{n}+\sum_{j=1}^{3}g^{ij}a_{j}^{n}\right)
\quad\quad\forall 0\leq m,n\leq 3.
\end{equation}
In particular, by \er{hoyuiouigyfg3478zzrrZZffhhhggjhhnnKKMM} and
\er{giuuihjghg78zzrrZZffhhhgguyuyyolhjhjhjjjkjjjjttnnKKMM} we
obtain:
\begin{multline}\label{giuuihjghg78zzrrZZffhhhgguyuyyolhjhjhjjjkjjjjttkklnnKKMM}
g'^{00}=a_{0}^{0}\left(g^{00}a_{0}^{0}+\sum_{j=1}^{3}g^{0j}a_{j}^{0}\right)+\sum_{i=1}^{3}a_{i}^{0}\left(g^{i0}a_{0}^{0}+\sum_{j=1}^{3}g^{ij}a_{j}^{0}\right)\\
=a_{0}^{0}\left(a_{0}^{0}+\sum_{j=1}^{3}2\frac{v^{j}}{c}a_{j}^{0}\right)+\left(\sum_{j=1}^{3}a_{j}^{0}\frac{v^j}{c}\right)^2-
\sum_{j=1}^{3}\left(a_{j}^{0}\right)^2=\left(a_{0}^{0}+\sum_{j=1}^{3}\frac{v^{j}}{c}a_{j}^{0}\right)^2-
\sum_{j=1}^{3}\left(a_{j}^{0}\right)^2 ,
\end{multline}
\begin{multline}\label{giuuihjghg78zzrrZZffhhhgguyuyyolhjhjhjjjkjjjjttjhnnKKMM}
g'^{0n}=g'^{n0}=a_{0}^{0}\left(g^{00}a_{0}^{n}+\sum_{j=1}^{3}g^{0j}a_{j}^{n}\right)+\sum_{i=1}^{3}a_{i}^{0}\left(g^{i0}a_{0}^{n}+\sum_{j=1}^{3}g^{ij}a_{j}^{n}\right)\\
=a_{0}^{0}\frac{v^{n}}{c}-a_{n}^{0}+\sum_{i=1}^{3}a_{i}^{0}\frac{v^{i}}{c}\frac{v^{n}}{c}
\quad\quad\forall 1\leq n\leq 3,
\end{multline}
and
\begin{multline}\label{giuuihjghg78zzrrZZffhhhgguyuyyolhjhjhjjjkjjjjttiiujklnnKKMM}
g'^{mn}=a_{0}^{m}\left(g^{00}a_{0}^{n}+\sum_{j=1}^{3}g^{0j}a_{j}^{n}\right)+\sum_{i=1}^{3}a_{i}^{m}\left(g^{i0}a_{0}^{n}+\sum_{j=1}^{3}g^{ij}a_{j}^{n}\right)
=\frac{v^{m}}{c}\frac{v^{n}}{c}-\delta_{mn} \quad\quad\forall 1\leq
m,n\leq 3.
\end{multline}
Thus by three equations together with
\er{hoyuiouigyfg3478zzrrZZffhhhggjhhnnKKMM} we deduce:
\begin{equation}\label{giuuihjghg78zzrrZZffhhhgguyuyyolhjhjhjjjkjjjjttkklhjhjnnKKMM}
\begin{cases}
g'^{00}=\left(1
+\sum_{j=1}^{3}\frac{1}{c^2}v^{j}\frac{\partial \Theta}{\partial
x^j}\right)^2- \sum_{j=1}^{3}\frac{1}{c^2}\left(\frac{\partial
\Theta}{\partial x^j}\right)^2 \\
g'^{0n}=g'^{n0} =
\frac{v^n}{c}\left(1+
\sum_{j=1}^{3}\frac{1}{c^2}v^j\frac{\partial \Theta}{\partial
x^j}\right)-\frac{1}{c}\frac{\partial \Theta}{\partial x^n}
\quad\quad\forall 1\leq n\leq 3,\\
g'^{mn} =\frac{v^{m}}{c}\frac{v^{n}}{c}-\delta_{mn}
\quad\quad\forall 1\leq m,n\leq 3.
\end{cases}
\end{equation}
%
%
%
%
%
%
Next if $\vec v$ satisfies \er{ggyfgdfdds} then choosing the
function $\Theta$ to be defined as:
\begin{equation}\label{ghfgfgfdfddfddf}
\Theta(\vec x)=\xi\left(|\vec x|\right)\quad\forall\,\vec
x,\quad\quad\text{where}\quad\quad
\frac{d\xi}{ds}(s)=\frac{g(s)}{1-\frac{g^2(s)}{c^2}}\quad\quad\forall\,s,
\end{equation}
we find that:
\begin{equation}\label{giuuihjghg78zzrrZZffhhhgguyuyyolhjhjhjjjkjjjjttjhopiuiohgghhjhjnnKKkkkkKKkkKKMM}
v_n=\left(1+
\sum_{j=1}^{3}\frac{1}{c^2}v^j\frac{\partial \Theta}{\partial x^j}
\right)^{-1}\frac{\partial \Theta}{\partial x^n} \quad\quad\forall
1\leq n\leq 3,
\end{equation}
i.e.
\begin{equation}\label{giuuihjghg78zzrrZZffhhhgguyuyyolhjhjhjjjkjjjjttjhopiuiohgghhjhjnnKKkkkkKKMM}
v_n\left(1+
\sum_{j=1}^{3}\frac{1}{c^2}v^j\frac{\partial \Theta}{\partial x^j}
\right)=\frac{\partial \Theta}{\partial x^n} \quad\quad\forall 1\leq
n\leq 3.
\end{equation}
Then we rewrite
\er{giuuihjghg78zzrrZZffhhhgguyuyyolhjhjhjjjkjjjjttkklhjhjnnKKMM}
as:
\begin{equation}\label{giuuihjghg78zzrrZZffhhhgguyuyyolhjhjhjjjkjjjjttkklhjhjhghgikkjnnKK11KKMM}
\begin{cases}
g'^{00}=\left(1-\frac{|\vec
v|^2}{c^2}\right)\left(1
+\sum_{j=1}^{3}\frac{1}{c^2}v^j\frac{\partial \Theta}{\partial x^j}
\right)^2 ,
\\
g'^{0n}=g'^{n0} =0 \quad\quad\forall 1\leq n\leq 3,
\\
g'^{mn}=\frac{v^m}{c}\frac{v^n}{c}-\delta_{mn} \quad\quad\forall
1\leq m,n\leq 3.
\end{cases}
\end{equation}
On the other hand by
\er{giuuihjghg78zzrrZZffhhhgguyuyyolhjhjhjjjkjjjjttjhopiuiohgghhjhjnnKKkkkkKKMM}
we have
\begin{equation}\label{giuuihjghg78zzrrZZffhhhgguyuyyolhjhjhjjjkjjjjttjhopiuiohgghhjhjnnKKkkkkKKkkkKKKJKKMM}
|\vec v|^2
=\left(1-\frac{|\vec
v|^2}{c^2}\right)\left(\sum_{j=1}^{3}v^j\frac{\partial
\Theta}{\partial x^j}\right).
\end{equation}
We rewrite
\er{giuuihjghg78zzrrZZffhhhgguyuyyolhjhjhjjjkjjjjttjhopiuiohgghhjhjnnKKkkkkKKkkkKKKJKKMM}
as:
\begin{equation}\label{giuuihjghg78zzrrZZffhhhgguyuyyolhjhjhjjjkjjjjttjhopiuiohgghhjhjnnKKkkkkKKkkkKKKJKKjkjjMM}
1
=\left(1-\frac{|\vec v|^2}{c^2}\right)\left(1
+\sum_{j=1}^{3}\frac{1}{c^2}v^j\frac{\partial \Theta}{\partial
x^j}\right)
\end{equation}
Therefore, by
\er{giuuihjghg78zzrrZZffhhhgguyuyyolhjhjhjjjkjjjjttkklhjhjhghgikkjnnKK11KKMM}
and
\er{giuuihjghg78zzrrZZffhhhgguyuyyolhjhjhjjjkjjjjttjhopiuiohgghhjhjnnKKkkkkKKkkkKKKJKKjkjjMM}
we deduce:
\begin{equation}\label{giuuihjghg78zzrrZZffhhhgguyuyyolhjhjhjjjkjjjjttkklhjhjhghgikkjnnKK11KKhjhhjhjKKMM}
\begin{cases}
g'^{00}=
\left(1-\frac{|\vec v|^2}{c^2}\right)^{-1} ,
\\
g'^{0n}=g'^{n0} =0 \quad\quad\forall 1\leq n\leq 3,
\\
g'^{mn}=\frac{v^m}{c}\frac{v^n}{c}-\delta_{mn} \quad\quad\forall
1\leq m,n\leq 3.
\end{cases}
\end{equation}
Next we find that the \underline{covariant} pseudo-metric tensor
$\{g'_{ij}\}_{0\leq i,j\leq 3}$ in the curvilinear coordinate system
$(**)$ has the following form:
\begin{equation}\label{giuuihjghg78zzrrZZffhhhgguyuyyolhjhjhjjjkjjjjttkklhjhjhghgikkjnnKK11KKhjhhjhjKKkkkllkKKMM}
\begin{cases}
g'_{00}=
\left(1-\frac{|\vec v|^2}{c^2}\right),
\\
g'_{0n}=g'_{n0} =0 \quad\quad\forall 1\leq n\leq 3,
\\
g'_{mn}=-\left(\left(1-\frac{|\vec
v|^2}{c^2}\right)^{-1}\frac{v^m}{c}\frac{v^n}{c}+\delta_{mn}\right)
\quad\quad\forall 1\leq m,n\leq 3.
\end{cases}
\end{equation}
Indeed, if $\{g'_{ij}\}_{0\leq i,j\leq 3}$ is defined by
\er{giuuihjghg78zzrrZZffhhhgguyuyyolhjhjhjjjkjjjjttkklhjhjhghgikkjnnKK11KKhjhhjhjKKkkkllkKKMM},
then by
\er{giuuihjghg78zzrrZZffhhhgguyuyyolhjhjhjjjkjjjjttkklhjhjhghgikkjnnKK11KKhjhhjhjKKMM}
we have:
\begin{equation*}
\sum_{k=0}^{3}g'_{0k}g'^{k0}=g'_{00}g'^{00}+\sum_{k=1}^{3}g'_{0k}g'^{k0}=1
,
\end{equation*}
\begin{multline*}
\sum_{k=0}^{3}g'_{ik}g'^{kj}=g'_{i0}g'^{0j}+\sum_{k=1}^{3}g'_{ik}g'^{kj}=
\sum_{k=1}^{3}\left(\left(1-\frac{|\vec
v|^2}{c^2}\right)^{-1}\frac{v^i}{c}\frac{v^k}{c}+\delta_{ik}\right)\left(\delta_{kj}-\frac{v^k}{c}\frac{v^j}{c}\right)\\=\delta_{ij}
-\left(1-\frac{|\vec v|^2}{c^2}\right)^{-1}\frac{|\vec
v|^2}{c^2}\frac{v^i}{c}\frac{v^j}{c}-\frac{v^i}{c}\frac{v^j}{c}+\left(1-\frac{|\vec
v|^2}{c^2}\right)^{-1}\frac{v^i}{c}\frac{v^j}{c}=\delta_{ij}
\quad\forall 1\leq i,j\leq 3,
\end{multline*}
and
\begin{equation*}
\sum_{k=0}^{3}g'_{ik}g'^{k0}=g'_{i0}g'^{00}+\sum_{k=1}^{3}g'_{ik}g'^{k0}=0
\quad\forall 1\leq i\leq 3,
\end{equation*}
\begin{equation*}
\sum_{k=0}^{3}g'_{0k}g'^{kj}=g'_{00}g'^{0j}+\sum_{k=1}^{3}g'_{0k}g'^{kj}=0
\quad\forall 1\leq j\leq 3.
\end{equation*}
So
\begin{equation*}
\sum_{k=0}^{3}g'_{ik}g'^{kj}=\delta_{ij} \quad\forall 0\leq i,j\leq
3,
\end{equation*}
and thus equalities
\er{giuuihjghg78zzrrZZffhhhgguyuyyolhjhjhjjjkjjjjttkklhjhjhghgikkjnnKK11KKhjhhjhjKKkkkllkKKMM}
indeed define the covariant form of the pseudo-metric tensor. So by
\er{giuuihjghg78zzrrZZffhhhgguyuyyolhjhjhjjjkjjjjttkklhjhjhghgikkjnnKK11KKhjhhjhjKKkkkllkKKMM}
we have:
\begin{equation}\label{giuuihjghg78zzrrZZffhhhgguyuyyolhjhjhjjjkjjjjttkklhjhjhghgikkjnnKK11KKhjhhjhjKKkkkllkKKMMioii}
\begin{cases}
g'_{00}=
\left(1-\frac{|\vec v|^2}{c^2}\right),
\\
g'_{0n}=g'_{n0} =0 \quad\quad\forall 1\leq n\leq 3,
\\
g'_{mn}=-\left(\left(1-\frac{|\vec
v|^2}{c^2}\right)^{-1}\frac{v^m}{c}\frac{v^n}{c}+\delta_{mn}\right)
\quad\quad\forall 1\leq m,n\leq 3.
\end{cases}
\end{equation}
In particular, the quadratic form, induced by the covariant form of
the pseudo-metric tensor $\{g'_{ij}\}_{0\leq i,j\leq 3}$ in the
curvilinear coordinate system $(**)$, that defined on the tangent
vectors $\left(dx'^0,dx'^1,dx'^2,dx'^3\right)\in\mathbb{R}^4$ where
$d\vec x':=(dx'^1,dx'^2,dx'^3)$ has the following form:
\begin{multline}\label{ghfgfgdfdfddffd}
\sum_{i=0}^{3}\sum_{j=0}^{3}g'_{ij}dx'^idx'^j=
\left(1-\frac{|\vec v|^2}{c^2}\right)dx'^2_0-\left(|d\vec
x'|^2+\frac{1}{c^2}\left(1-\frac{|\vec
v|^2}{c^2}\right)^{-1}\left|\vec v\cdot d\vec x'\right|^2\right)=\\
\left(1-\frac{|\vec v|^2}{c^2}\right)dx'^2_0-\left(\left(|d\vec
x'|^2-\left|\frac{\vec v}{|\vec v|}\cdot d\vec
x'\right|^2\right)+\left(1-\frac{|\vec
v|^2}{c^2}\right)^{-1}\frac{|\vec v|^2}{c^2}\,\left|\frac{\vec
v}{|\vec v|}\cdot d\vec x'\right|^2+\left|\frac{\vec v}{|\vec
v|}\cdot d\vec x'\right|^2\right)\\=
\left(1-\frac{|\vec
v|^2}{c^2}\right)dx'^2_0-\left(\left(1-\frac{|\vec
v|^2}{c^2}\right)^{-1}\left|\frac{\vec v}{|\vec v|}\cdot d\vec
x'\right|^2+\left(|d\vec x'|^2-\left|\frac{\vec v}{|\vec v|}\cdot
d\vec x'\right|^2\right)\right).
\end{multline}
Thus taking into account
\er{giuuihjghgghjgj78zzrrZZffhhhggygghghhjjknnKKMM} and
\er{ggyfgdfdds} we rewrite \er{ghfgfgdfdfddffd} as:
\begin{multline}\label{ghfgfgdfdfddffdhhjhj}
\sum_{i=0}^{3}\sum_{j=0}^{3}g'_{ij}dx'^idx'^j=\\
\left(1-\frac{\left|\vec v(\vec
x')\right|^2}{c^2}\right)dx'^2_0-\left(\left(1-\frac{\left|\vec
v(\vec x')\right|^2}{c^2}\right)^{-1}\left|\frac{\vec x'}{|\vec
x'|}\cdot d\vec x'\right|^2+\left(|d\vec x'|^2-\left|\frac{\vec
x'}{|\vec x'|}\cdot d\vec x'\right|^2\right)\right).
\end{multline}

Next, up to the end of this subsection, assume that our cartesian
coordinate system $(*)$ is non-rotating and our gravitational field
is formed by the spherical symmetric massive body of mass $m_0$ and
radius $R_0$ like the Earth, the Sun et.al. with the center at the
point $0$. Then, as we get either in
\er{MaxVacFull1ninshtrgravortghhghgjkgghklhjgkghghjjkjhjkkggjkhjkhjjhhfhjhklkhkhjjklzzzyyyhjggjhgghhjhNWNWBWHWPPN222kkkgtytghjjhpoppkkkhhhkhjhj}
and
\er{MaxVacFull1ninshtrgravortghhghgjkgghklhjgkghghjjkjhjkkggjkhjkhjjhhfhjhklkhkhjjklzzzyyyhjggjhgghhjhNWNWBWHWPPN222kkkgtytghjjhpoppkkkhhhkhjhjiuu}
of Remark \ref{gygygygyggjh} in the case of the Newtonian gravity,
or in
\er{MaxVacFull1ninshtrgravortghhghgjkgghklhjgkghghjjkjhjkkggjkhjkhjjhhfhjhklkhkhjjklzzzyyyhjggjhgghhjhNWNWBWHWPPN222kkkgtytghjjhpoppkkkhhhkhjhj34}
and
\er{MaxVacFull1ninshtrgravortghhghgjkgghklhjgkghghjjkjhjkkggjkhjkhjjhhfhjhklkhkhjjklzzzyyyhjggjhgghhjhNWNWBWHWPPN222kkkgtytghjjhpoppkkkhhhkhjhjiuu34}
of Remark \ref{ghjfhgfdhdnw22mm} in the case of the more general
gravity model, given by \er{guigjgjffghguygjyfDDnw22mm} with
arbitrary constant $\beta$, we have: either
\begin{equation}
\label{MaxVacFull1ninshtrgravortghhghgjkgghklhjgkghghjjkjhjkkggjkhjkhjjhhfhjhklkhkhjjklzzzyyyhjggjhgghhjhNWNWBWHWPPN222kkkgtytghjjhpoppkkkhhhkhjhjhjhjj}
\vec v(\vec x)= \frac{\sqrt{-2\Phi_1(|\vec x|)}}{|\vec x|}\vec x,
\end{equation}
or
\begin{equation}
\label{MaxVacFull1ninshtrgravortghhghgjkgghklhjgkghghjjkjhjkkggjkhjkhjjhhfhjhklkhkhjjklzzzyyyhjggjhgghhjhNWNWBWHWPPN222kkkgtytghjjhpoppkkkhhhkhjhjiuuiuhhj}
\vec v(\vec x)= -\frac{\sqrt{-2\Phi_1(|\vec x|)}}{|\vec x|}\vec x,
\end{equation}
where $\Phi_1$ is the classical Newtonian potential of our massive
body $m_0$ that satisfies
\begin{equation}\label{ghgffgfgfgf}
\Phi_1(\vec x)=-\frac{Gm_0}{|\vec x|}
\end{equation}
outside of the body surface. Thus in particular,
\begin{equation}\label{ghgffgfgfgfbbbojkjk}
\left|\vec v(\vec x)\right|^2=-2\Phi_1(|\vec x|),
\end{equation}
and outside of the massive body surface we have:
\begin{equation}\label{ghgffgfgfgfbbb}
\left|\vec v(\vec x)\right|^2=\frac{2Gm_0}{|\vec x|}.
\end{equation}
Both
\er{MaxVacFull1ninshtrgravortghhghgjkgghklhjgkghghjjkjhjkkggjkhjkhjjhhfhjhklkhkhjjklzzzyyyhjggjhgghhjhNWNWBWHWPPN222kkkgtytghjjhpoppkkkhhhkhjhjhjhjj}
and
\er{MaxVacFull1ninshtrgravortghhghgjkgghklhjgkghghjjkjhjkkggjkhjkhjjhhfhjhklkhkhjjklzzzyyyhjggjhgghhjhNWNWBWHWPPN222kkkgtytghjjhpoppkkkhhhkhjhjiuuiuhhj}
are particular cases of \er{ggyfgdfdds}, with
\begin{equation}
\label{MaxVacFull1ninshtrgravortghhghgjkgghklhjgkghghjjkjhjkkggjkhjkhjjhhfhjhklkhkhjjklzzzyyyhjggjhgghhjhNWNWBWHWPPN222kkkgtytghjjhpoppkkkhhhkhjhjiuuiuhhjjkkkmn}
g(s)=\pm\sqrt{-2\Phi_1(s)},
\end{equation}
and in particular, outside of the massive body surface we have:
\begin{equation}
\label{MaxVacFull1ninshtrgravortghhghgjkgghklhjgkghghjjkjhjkkggjkhjkhjjhhfhjhklkhkhjjklzzzyyyhjggjhgghhjhNWNWBWHWPPN222kkkgtytghjjhpoppkkkhhhkhjhjiuuiuhhjjkkkmnhjjh}
g\left(|x|\right)=\pm\sqrt{\frac{2Gm_0}{|\vec x|}},
\end{equation}
Thus defining the function $\Theta(\vec x)$ as in
\er{ghfgfgfdfddfddf}, that always can be done in the case
$\frac{2Gm_0}{R_0}<c^2$, we can define the change of variables from
coordinate system $(*)$ to the curvilinear coordinate system $(**)$
in the four-dimensional space-time $\mathbb{R}^4$ as in
\er{giuuihjghgghjgj78zzrrZZffhhhggygghghhjjknnKKMM}:
\begin{equation}\label{giuuihjghgghjgj78zzrrZZffhhhggygghghhjjknnKKMMih}
\begin{cases}
t'=t+\frac{\Theta(\vec x)}{c^2}\\
\vec x'=\vec x.
\end{cases}
\end{equation}
Then by inserting
\er{MaxVacFull1ninshtrgravortghhghgjkgghklhjgkghghjjkjhjkkggjkhjkhjjhhfhjhklkhkhjjklzzzyyyhjggjhgghhjhNWNWBWHWPPN222kkkgtytghjjhpoppkkkhhhkhjhjhjhjj}
or
\er{MaxVacFull1ninshtrgravortghhghgjkgghklhjgkghghjjkjhjkkggjkhjkhjjhhfhjhklkhkhjjklzzzyyyhjggjhgghhjhNWNWBWHWPPN222kkkgtytghjjhpoppkkkhhhkhjhjiuuiuhhj}
into
\er{giuuihjghg78zzrrZZffhhhgguyuyyolhjhjhjjjkjjjjttkklhjhjhghgikkjnnKK11KKhjhhjhjKKkkkllkKKMMioii}
we deduce the form of the covariant pseudo-metric tensor in the
curvilinear coordinate system $(**)$:
\begin{equation}\label{giuuihjghg78zzrrZZffhhhgguyuyyolhjhjhjjjkjjjjttkklhjhjhghgikkjnnKK11KKhjhhjhjKKkkkllkKKMMioiijjjkjk}
\begin{cases}
g'_{00}=
\left(1+\frac{2\Phi_1(|\vec x'|)}{c^2}\right),
\\
g'_{0n}=g'_{n0} =0 \quad\quad\forall 1\leq n\leq 3,
\\
g'_{mn}=\left(\left(1+\frac{2\Phi_1(|\vec
x'|)}{c^2}\right)^{-1}\frac{2\Phi_1(|\vec
x'|)}{c^2}\frac{x'_m}{|\vec x'|}\frac{x'_n}{|\vec
x'|}-\delta_{mn}\right) \quad\quad\forall 1\leq m,n\leq 3.
\end{cases}
\end{equation}
Moreover, by \er{ghfgfgdfdfddffdhhjhj} we have:
\begin{multline}\label{ghfgfgdfdfddffdhhjhjuiuuui}
\sum_{i=0}^{3}\sum_{j=0}^{3}g'_{ij}dx'^idx'^j=\\
\left(1+\frac{2\Phi_1(|\vec
x'|)}{c^2}\right)dx'^2_0-\left(\left(1+\frac{2\Phi_1(|\vec
x'|)}{c^2}\right)^{-1}\left|\frac{\vec x'}{|\vec x'|}\cdot d\vec
x'\right|^2+\left(|d\vec x'|^2-\left|\frac{\vec x'}{|\vec x'|}\cdot
d\vec x'\right|^2\right)\right).
\end{multline}
In particular, outside of the massive body surface, i.e. when
$|x'|>R_0$ we rewrite
\er{giuuihjghg78zzrrZZffhhhgguyuyyolhjhjhjjjkjjjjttkklhjhjhghgikkjnnKK11KKhjhhjhjKKkkkllkKKMMioiijjjkjk}
and \er{ghfgfgdfdfddffdhhjhjuiuuui} as:
\begin{equation}\label{giuuihjghg78zzrrZZffhhhgguyuyyolhjhjhjjjkjjjjttkklhjhjhghgikkjnnKK11KKhjhhjhjKKkkkllkKKMMioiijjjkjkuyuh}
\begin{cases}
g'_{00}=
\left(1-\frac{2Gm_0}{c^2|\vec x'|}\right),
\\
g'_{0n}=g'_{n0} =0 \quad\quad\forall 1\leq n\leq 3,
\\
g'_{mn}=-\left(\left(1-\frac{2Gm_0}{c^2|\vec
x'|}\right)^{-1}\frac{2Gm_0}{c^2|\vec x'|}\frac{x'_m}{|\vec
x'|}\frac{x'_n}{|\vec x'|}+\delta_{mn}\right) \quad\quad\forall
1\leq m,n\leq 3,
\end{cases}
\end{equation}
and
\begin{multline}\label{ghfgfgdfdfddffdhhjhjuiuuuikkj}
\sum_{i=0}^{3}\sum_{j=0}^{3}g'_{ij}dx'^idx'^j=\\
\left(1-\frac{2Gm_0}{c^2|\vec
x'|}\right)dx'^2_0-\left(\left(1-\frac{2Gm_0}{c^2|\vec
x'|}\right)^{-1}\left|\frac{\vec x'}{|\vec x'|}\cdot d\vec
x'\right|^2+\left(|d\vec x'|^2-\left|\frac{\vec x'}{|\vec x'|}\cdot
d\vec x'\right|^2\right)\right).
\end{multline}
Therefore, we get that in coordinate system $(**)$, outside of the
massive body, the covariant pseudo-metric tensor in
\er{giuuihjghg78zzrrZZffhhhgguyuyyolhjhjhjjjkjjjjttkklhjhjhghgikkjnnKK11KKhjhhjhjKKkkkllkKKMMioiijjjkjkuyuh}
and \er{ghfgfgdfdfddffdhhjhjuiuuuikkj} \underline{exactly} the same
as the well known Schwarzschild metric from the General Relativity.
Indeed in the spherical coordinates in $\mathbb{R}^3$ we rewrite
\er{ghfgfgdfdfddffdhhjhjuiuuuikkj} as:
\begin{multline}\label{ghfgfgdfdfddffdhhjhjuiuuuikkjjhhjhj}
\sum_{i=0}^{3}\sum_{j=0}^{3}g'_{ij}dx'^idx'^j=\\
\left(1-\frac{2Gm_0}{c^2r'}\right)dx'^2_0-\left(\left(1-\frac{2Gm_0}{c^2r'}\right)^{-1}\left(dr'\right)^2+(r')^2\left((d\theta')^2+\sin^2{\left(\theta'\right)}(d\varphi')^2\right)\right),
\end{multline}
and this is exactly the classical Schwarzschild metric.

In particular, if we consider the monochromatic electromagnetic wave
of frequency $\omega$ of the form $e^{i\omega t}U(\vec x)$ in the
coordinate system $(*)$, then by
\er{giuuihjghgghjgj78zzrrZZffhhhggygghghhjjknnKKMMih} in the
coordinate system $(**)$ the form of this light is $e^{i\omega
t'}U'(\vec x')$ where $U'(\vec x')=U(\vec
x')e^{-i\omega\frac{\Theta(\vec x')}{c^2}}$, i.e the electromagnetic
wave in the coordinate system  $(**)$ is also monochromatic of the
same frequency $\omega$. Thus all the optical effects that we find
in the frames of our model coincides with the effects considered in
the frames of General Relativity for the Schwarzschild metric. In
particular, the Michelson-Morely experiment and all Sagnac-type
effects will lead to the same result in the frame of our model like
in the case of the General relativity. Moreover, since the Maxwell
equations in both models have the same tensor form, all the
electromagnetic effects, where the time does not appear explicitly
will be the same. Similarly, the curvature of the light path in the
Sun's gravitational field will be the same in both models. Finally,
in the particular case of $\mathcal{G}(\tau)=\sqrt{\tau}$ in
\er{btfffygtgyggyijhhkkSYSPNuhuygyygyggyuyyuyuykuhghgjjojhjfg}, i.e.
in the case of the relativistic-like Lagrangian of the motion in
\er{btfffygtgyggyijhhkkSYSPNuhuygyygyggyuyyuyuykuhghgjjojiyyyuyuuyy}
all the mechanical effects will be the same in the frame of our
model like in the case of the General relativity  for the
Schwarzschild metric, provided that the time does not appear
explicitly in this effects. In particular, the movement of the
Mercury planet in the Sun's gravitational field will be the same in
both models, provided we take into account the relativistic-like
Lagrangian of the motion as in
\er{btfffygtgyggyijhhkkSYSPNuhuygyygyggyuyyuyuykuhghgjjojiyyyuyuuyy}.

Finally, note that the similar, solution as in
\er{ghfgfgdfdfddffdhhjhjuiuuuikkj} or
\er{ghfgfgdfdfddffdhhjhjuiuuuikkjjhhjhj} is valid also for the
general laws of the gravity, given by either
\er{MaxVacFull1bjkgjhjhgjaaahkjhhjzzzyyykkknnnNWBWHWPPNNewCCmmhhjjhkkk}
or \er{guigjgjffghguygjyfDDnw22mm34}, where in all places we take
$M_0$ instead of $m_0$ and $\Phi_1(\vec x)=-\frac{GM_0}{|\vec x|}$
instead of \er{ghgffgfgfgf}, with $M_0$ being the total effective
\underline{gravitational} mass of the Earth, defined as in
\er{MaxVacFull1ninshtrgravortghhghgjkgghklhjgkghghjjkjhjkkggjkhjkhjjhhfhjhklkhkhjjklzzzyyyhjggjhgghhjhNWNWBWHWPPN222kkkgtytghjjhpoppkkkhhhk3499jhikjhjhj99jkkh}
(see subsection \ref{ghjfhgfdhdnw22mm34} for the details).

\subsection{General Lagraingian of the gravitational-electromagnetic
field, compatible with the general Lagrangian of the motion in
\er{btfffygtgyggyijhhkkSYSPNuhuygyygyggyuyyuyuykuhghgjjoj}}\label{hughjughghffghfgffg1}

\subsubsection{The case of the Newtonian-type gravity}\label{hughjughghffghfgffg}
Given known the distribution of inertial mass density of some
continuum medium $\mu:=\mu(\vec x,t)$, the field of velocities of
this medium $\vec u:=\vec u(\vec x,t)$, the charge density
$\rho:=\rho(\vec x,t)$ and the current density $\vec j:=\vec j(\vec
x,t)$ in a cartesian coordinate system,  consider a general
Lagrangian density $L$ of the unified gravitational-electromagnetic
field that generalize the Lagrangian density defined by
\er{vhfffngghkjgghDD} and is consistent with the Lagrangian of the
motion of particles of the general form
\er{btfffygtgyggyijhhkkSYSPNuhuygyygyggyuyyuyuykuhghgjjoj}:
\begin{multline}\label{vhfffngghkjgghDDmmjj}
L\left(\vec A,\Psi,\vec v,\Phi,\vec p,\vec
x,t\right):=\frac{1}{8\pi}\left|-\nabla_{\vec
x}\Psi-\frac{1}{c}\frac{\partial\vec A}{\partial t}+\frac{1}{c}\vec
v\times curl_{\vec x}\vec A\right|^2-\frac{1}{8\pi}\left|curl_{\vec
x}\vec A\right|^2-\left(\rho\Psi-\frac{1}{c}\vec A\cdot\vec
j\right)\\-\mu c^2 \mathcal{G}\left(1-\frac{1}{c^2}\left|\vec u-\vec
v\right|^2\right)+
\frac{1}{2}\left(d_{\vec x}\vec v+\left\{d_{\vec x}\vec
v\right\}^T\right):\left(d_{\vec x}\vec p+\left\{d_{\vec x}\vec
p\right\}^T\right)-2\left(div_{\vec x}\vec v\right)\left(div_{\vec
x}\vec p\right)\\+\frac{1}{4\pi G}\left(div_{\vec x}\vec
v\right)\left(\frac{\partial \Phi}{\partial t}+\vec
v\cdot\nabla_{\vec x} \Phi\right)+\frac{1}{4\pi
G}\Phi\left(div_{\vec x}\vec v\right)^2-\frac{\Phi}{16\pi
G}\left|d_{\vec x}\vec v+\left\{d_{\vec x}\vec
v\right\}^T\right|^2+\frac{1}{8\pi G}\left|\nabla_{\vec
x}\Phi\right|^2,
\end{multline}
where $\mathcal{G}(s):\mathbb{R}\to\mathbb{R}$ is a given function,
$\Phi$ is some ancillary proper scalar field and $\vec p$ is some
ancillary proper vector field. Then $L$ is invariant under the
change of inertial or non-inertial cartesian coordinate system of
the form \er{noninchgravortbstrjgghguittu2intrrrZZygjyg}. Then
denoting the function
\begin{equation}\label{ujghjjghjh}
g(s):=-c^2\mathcal{G}\left(1-\frac{2s}{c^2}\right)\quad\quad\forall
s
\end{equation}
we rewrite \er{vhfffngghkjgghDDmmjj} as:
\begin{multline}\label{vhfffngghkjgghDDmm}
L\left(\vec A,\Psi,\vec v,\Phi,\vec p,\vec
x,t\right):=\frac{1}{8\pi}\left|-\nabla_{\vec
x}\Psi-\frac{1}{c}\frac{\partial\vec A}{\partial t}+\frac{1}{c}\vec
v\times curl_{\vec x}\vec A\right|^2-\frac{1}{8\pi}\left|curl_{\vec
x}\vec A\right|^2-\left(\rho\Psi-\frac{1}{c}\vec A\cdot\vec
j\right)\\+\mu g\left(\frac{1}{2}\left|\vec u-\vec
v\right|^2\right)+
\frac{1}{2}\left(d_{\vec x}\vec v+\left\{d_{\vec x}\vec
v\right\}^T\right):\left(d_{\vec x}\vec p+\left\{d_{\vec x}\vec
p\right\}^T\right)-2\left(div_{\vec x}\vec v\right)\left(div_{\vec
x}\vec p\right)\\+\frac{1}{4\pi G}\left(div_{\vec x}\vec
v\right)\left(\frac{\partial \Phi}{\partial t}+\vec
v\cdot\nabla_{\vec x} \Phi\right)+\frac{1}{4\pi
G}\Phi\left(div_{\vec x}\vec v\right)^2-\frac{\Phi}{16\pi
G}\left|d_{\vec x}\vec v+\left\{d_{\vec x}\vec
v\right\}^T\right|^2+\frac{1}{8\pi G}\left|\nabla_{\vec
x}\Phi\right|^2.
\end{multline}
We point out two the most important choices of function
$\mathcal{G}(s)$: fully non-relativistic choice
$\mathcal{G}(s)=\frac{s}{2}$ and correspondingly
$g(s)=\left(s-\frac{c^2}{2}\right)$; and relativistic-like choice
$\mathcal{G}(s)=\sqrt{s}$ and correspondingly
$g(s):=-c^2\sqrt{1-\frac{2s}{c^2}}$. Note also that in the first
case we have $\frac{dg}{ds}(s)=1$ and in the second case
$\frac{dg}{ds}(s)=\left(1-\frac{2s}{c^2}\right)^{-\frac{1}{2}}\approx
1$, where the last equation is valid in the case where $2s\ll c^2$.

We investigate stationary points of the functional
\begin{equation}\label{btfffygtgyggyDDmm}
J=\int_0^T\int_{\mathbb{R}^3}L\left(\vec A,\Psi,\vec v,\Phi,\vec
p,\vec x,t\right)d\vec x dt.
\end{equation}
We denote
\begin{equation}\label{guigjgjffghDDmm}
\begin{cases}
\vec D=-\nabla_{\vec x}\Psi-\frac{1}{c}\frac{\partial\vec
A}{\partial t}+\frac{1}{c}\vec
v\times curl_{\vec x}\vec A\\
\vec B=curl_{\vec x}\vec A
\\
\vec E=-\nabla_{\vec x}\Psi-\frac{1}{c}\frac{\partial\vec A}{\partial t}=\vec D-\frac{1}{c}\vec v\times\vec B\\
\vec H=curl_{\vec x}\vec A+\frac{1}{c}\vec
v\times\left(-\nabla_{\vec x}\Psi-\frac{1}{c}\frac{\partial\vec
A}{\partial t}+\frac{1}{c}\vec v\times curl_{\vec x}\vec
A\right)=\vec B+\frac{1}{c}\vec v\times\vec D.
\end{cases}
\end{equation}
Then by \er{guigjgjffghDDmm} we have:
\begin{equation}\label{guigjgjffghjhkkgDDmm}
\begin{cases}
curl_{\vec x}\vec E+\frac{1}{c}\frac{\partial\vec B}{\partial t}=0\\
div_{\vec x}\vec B=0.
\end{cases}
\end{equation}
Moreover by \er{vhfffngghkjgghDDmm}, \er{apfrm12}
and \er{apfrm6} we
have
\begin{equation}\label{vhfffngghkjgghggtghjgfhjyuiutDDmm}
\frac{\delta L}{\delta \vec p}=-div_{\vec x}\left(d_{\vec x}\vec
v+\left\{d_{\vec x}\vec v\right\}^T\right)+2\nabla_{\vec
x}\left(div_{\vec x}\vec v\right)=curl_{\vec x}\left(curl_{\vec
x}\vec v\right)=0,
\end{equation}
\begin{equation}\label{vhfffngghkjgghggtghjgfhjDDmm}
\frac{\delta L}{\delta \Phi}=-\frac{1}{4\pi
G}\left(\frac{\partial}{\partial t}\left\{div_{\vec x}\vec
v\right\}+\vec v\cdot\nabla_{\vec x}\left(div_{\vec x}\vec
v\right)+\frac{1}{4}\left|d_{\vec x}\vec v+\left\{d_{\vec x}\vec
v\right\}^T\right|^2\right)-\frac{1}{4\pi G}\Delta_{\vec x}\Phi=0,
\end{equation}
\begin{multline}\label{vhfffngghkjgghggtghjgfhjoyuiyuDDmm}
\frac{\delta L}{\delta \vec v}=-\left(\mu
g'\left(\frac{1}{2}\left|\vec u-\vec v\right|^2\right)(\vec u-\vec
v)+\frac{1}{4\pi c}\vec D\times\vec B\right)-div_{\vec
x}\left(d_{\vec x}\vec p+\left\{d_{\vec x}\vec
p\right\}^T\right)+2\nabla_{\vec x}\left(div_{\vec x}\vec
p\right)\\+\frac{1}{4\pi G}div_{\vec x}\left\{\left(d_{\vec x}\vec
v+\left\{d_{\vec x}\vec v\right\}^T\right)\Phi\right\}-\frac{1}{2\pi
G}\nabla_{\vec x}\left(\Phi\left(div_{\vec x}\vec
v\right)\right)-\frac{1}{4\pi G}\nabla_{\vec x}\left(\frac{\partial
\Phi}{\partial t}+\vec v\cdot\nabla_{\vec x}
\Phi\right)\\+\frac{1}{4\pi G}\left(div_{\vec x}\vec
v\right)\nabla_{\vec x}\Phi=-\left(\mu g'\left(\frac{1}{2}\left|\vec
u-\vec v\right|^2\right)(\vec u-\vec v)+\frac{1}{4\pi c}\vec
D\times\vec B\right)+curl_{\vec x}\left(curl_{\vec x}\vec
p\right)\\-\frac{1}{4\pi G}\Phi\,curl_{\vec x}\left(curl_{\vec
x}\vec v\right)-\frac{1}{4\pi G}\left(\frac{\partial}{\partial
t}\left(\nabla_{\vec x}\Phi\right)-curl_{\vec x}\left(\vec
v\times\nabla_{\vec x}\Phi\right)+\left(\Delta_{\vec
x}\Phi\right)\vec v\right)=0,
\end{multline}
\begin{equation}\label{vhfffngghkjgghggtghjgfhjhjkghghDDmm}
\frac{\delta L}{\delta \Psi}=\frac{1}{4\pi}div_{\vec x}\vec
D-\rho=0,
\end{equation}
and
\begin{equation}\label{vhfffngghkjgghggtghjgfhjhjkghghyuiuuDDmm}
\frac{\delta L}{\delta \vec A}=\frac{1}{c}\vec j+\frac{1}{4\pi
c}\frac{\partial\vec D}{\partial t}-\frac{1}{4\pi}curl_{\vec x}\vec
B-\frac{1}{4\pi c}curl_{\vec x}\left(\vec v\times \vec
D\right)=\frac{1}{c}\vec j+\frac{1}{4\pi c}\frac{\partial\vec
D}{\partial t}-\frac{1}{4\pi}curl_{\vec x}\vec H=0.
\end{equation}
So
\begin{equation}\label{guigjgjffghguygjyfDDmm}
\begin{cases}
curl_{\vec x}\vec H=\frac{4\pi}{c}\vec
j+\frac{1}{c}\frac{\partial\vec D}{\partial
t}\\
div_{\vec x}\vec D=4\pi\rho\\
curl_{\vec x}\vec E+\frac{1}{c}\frac{\partial\vec B}{\partial t}=0\\
div_{\vec x}\vec B=0\\
\vec E=\vec D-\frac{1}{c}\vec v\times\vec B\\
\vec H=\vec B+\frac{1}{c}\vec v\times\vec D\\
curl_{\vec x}\left(curl_{\vec x}\vec v\right)=0
\\
\frac{\partial}{\partial t}\left\{div_{\vec x}\vec v\right\}+\vec
v\cdot\nabla_{\vec x}\left(div_{\vec x}\vec
v\right)+\frac{1}{4}\left|d_{\vec x}\vec v+\left\{d_{\vec x}\vec
v\right\}^T\right|^2=-\Delta_{\vec x}\Phi\\
\left(\mu g'\left(\frac{1}{2}\left|\vec u-\vec
v\right|^2\right)(\vec u-\vec v)+\frac{1}{4\pi c}\vec D\times\vec
B\right)\\=curl_{\vec x}\left(curl_{\vec x}\vec
p\right)-\frac{1}{4\pi G}\left(\frac{\partial}{\partial
t}\left(\nabla_{\vec x}\Phi\right)-curl_{\vec x}\left(\vec
v\times\nabla_{\vec x}\Phi\right)+\left(\Delta_{\vec
x}\Phi\right)\vec v\right).
\end{cases}
\end{equation}

Next consider the equations of the gravitational-electromagnetic
field in the form \er{guigjgjffghguygjyfDDmm}. Then defining the
gravitational mass
\begin{equation}\label{Mjghjghhjg}
M:=\frac{1}{4\pi G}\Delta_{\vec x}\Phi
\end{equation} we obtain
\begin{equation}\label{MaxVacFull1bjkgjhjhgjaaahkjhhjzzzyyykkknnnNWBWHWPPNNewCCmmhhjjh}
\begin{cases}
curl_{\vec x}\vec H=\frac{4\pi}{c}\vec
j+\frac{1}{c}\frac{\partial\vec D}{\partial
t},\\
div_{\vec x}\vec D=4\pi\rho,\\
curl_{\vec x}\vec E+\frac{1}{c}\frac{\partial\vec B}{\partial t}=0,\\
div_{\vec x}\vec B=0,\\
\vec E=\vec D-\frac{1}{c}\vec v\times\vec B,\\
\vec H=\vec B+\frac{1}{c}\vec v\times\vec D,\\
\frac{\partial}{\partial t}\left\{div_{\vec x}\vec v\right\}+\vec
v\cdot\nabla_{\vec x}\left(div_{\vec x}\vec
v\right)+\frac{1}{4}\left|d_{\vec x}\vec v+\left\{d_{\vec x}\vec
v\right\}^T\right|^2=-4\pi GM,\\
 curl_{\vec
x}\left(curl_{\vec x}\vec v\right)=0,
\\
\frac{\partial M}{\partial t}+div_{\vec x}\left\{M\vec v+\mu
g'\left(\frac{1}{2}\left|\vec u-\vec v\right|^2\right)(\vec u-\vec
v)+\frac{1}{4\pi c}\vec D\times\vec B\right\}=0.
\end{cases}
\end{equation}
Using the continuum equation
\begin{equation}\label{MaxVacFull1ninshtrgravortghhghgjkgghklhjgkghghjjkjhjkkggjkhjkhjjhhfhjhkhjjkhjgzzzyyykkknnnNWBWHWPPNNewZZZCCmmhjhj}
\frac{\partial\mu}{\partial t}+div_{\vec x}\left(\mu\vec u\right)=0.
\end{equation}
we rewrite
\er{MaxVacFull1bjkgjhjhgjaaahkjhhjzzzyyykkknnnNWBWHWPPNNewCCmmhhjjh}
as
\begin{equation}\label{MaxVacFull1bjkgjhjhgjaaahkjhhjzzzyyykkknnnNWBWHWPPNNewCCmmhhjjhkkk}
\begin{cases}
curl_{\vec x}\vec H=\frac{4\pi}{c}\vec
j+\frac{1}{c}\frac{\partial\vec D}{\partial
t},\\
div_{\vec x}\vec D=4\pi\rho,\\
curl_{\vec x}\vec E+\frac{1}{c}\frac{\partial\vec B}{\partial t}=0,\\
div_{\vec x}\vec B=0,\\
\vec E=\vec D-\frac{1}{c}\vec v\times\vec B,\\
\vec H=\vec B+\frac{1}{c}\vec v\times\vec D,\\
\frac{\partial}{\partial t}\left\{div_{\vec x}\vec v\right\}+\vec
v\cdot\nabla_{\vec x}\left(div_{\vec x}\vec
v\right)+\frac{1}{4}\left|d_{\vec x}\vec v+\left\{d_{\vec x}\vec
v\right\}^T\right|^2=-4\pi GM,\\
 curl_{\vec
x}\left(curl_{\vec x}\vec v\right)=0,
\\
\frac{\partial}{\partial t}\left(M-\mu\right)+div_{\vec
x}\left\{\left(M-\mu\right)\vec v\right\}=-div_{\vec x}\left\{\mu
\left(g'\left(\frac{1}{2}\left|\vec u-\vec
v\right|^2\right)-1\right)(\vec u-\vec v)+\frac{1}{4\pi c}\vec
D\times\vec B\right\}\,.
\end{cases}
\end{equation}
Note again for the last equation in
\er{MaxVacFull1bjkgjhjhgjaaahkjhhjzzzyyykkknnnNWBWHWPPNNewCCmmhhjjhkkk}
that: in the fully non-relativistic case we have $g'(s)=1$ and in
the relativistic-like case we have
$g'(s)=\left(1-\frac{2s}{c^2}\right)^{-\frac{1}{2}}\approx 1$, where
the last equation is valid in the case where $2s\ll c^2$.

\subsubsection{The case of some possible
alternative model of the gravity}\label{ghjfhgfdhdnw22mm34} Consider
$\vec k$ to be the vectorial potential of the inertia, which is a
generally trivial speed-like vector field, assumed to be fixed in
every fixed inertial or non-inertial cartesian coordinate system
(see Definition \ref{jjhgyftdtd}). Given known the distribution of
inertial mass density of some continuum medium $\mu:=\mu(\vec x,t)$,
the field of velocities of this medium $\vec u:=\vec u(\vec x,t)$,
the charge density $\rho:=\rho(\vec x,t)$ and the current density
$\vec j:=\vec j(\vec x,t)$ consider a Lagrangian density $L$ defined
by
\begin{multline}\label{vhfffngghkjgghDDnw22mm34mm34}
L\left(\vec A,\Psi,\vec v,\Phi_0,\vec h,\vec x,t\right):=\\
\frac{1}{8\pi}\left|-\nabla_{\vec
x}\Psi-\frac{1}{c}\frac{\partial\vec A}{\partial t}+\frac{1}{c}\vec
v\times curl_{\vec x}\vec A\right|^2-\frac{1}{8\pi}\left|curl_{\vec
x}\vec A\right|^2-\left(\rho\Psi-\frac{1}{c}\vec A\cdot\vec j\right)
-\mu c^2 \mathcal{G}\left(1-\frac{1}{c^2}\left|\vec u-\vec
v\right|^2\right)
\\
-\frac{c^2}{8\pi G}\left|-\nabla_{\vec x}\Phi_0
-\frac{1}{c}\frac{\partial\vec h}{\partial t}
+\frac{1}{c}\vec v\times curl_{\vec x}\vec h-\frac{1}{c}\nabla_{\vec
x}\left(\vec h\cdot\vec v\right)\right|^2+\frac{c^2}{8\pi
G}\left|curl_{\vec x}\vec h\right|^2\\+\frac{c^2\beta}{4\pi
G}\left(\frac{1}{4}\left|d_{\vec x}\vec v+\left\{d_{\vec x}\vec
v\right\}^T\right|^2-\left|\Div_{\vec x}\vec v\right|^2\right) ,
\end{multline}
where
\begin{equation}\label{btfffygtgyggyDDnw22jkhk22mm34}
\vec h=\vec v-\vec
k\quad\text{and}\quad\Phi_0=-\frac{1}{2c}\left|\vec v-\vec
k\right|^2,
\end{equation}
$\mathcal{G}(s):\mathbb{R}\to\mathbb{R}$ is a given function, and
$\beta\in \mathbb{R}$ is some constant. Then, as before, denoting
the function
\begin{equation}\label{ujghjjghjh34}
g(s):=-c^2\mathcal{G}\left(1-\frac{2s}{c^2}\right)\quad\quad\forall
s
\end{equation}
we rewrite \er{vhfffngghkjgghDDnw22mm34mm34} as:
\begin{multline}\label{vhfffngghkjgghDDnw22mm34}
L\left(\vec A,\Psi,\vec v,\Phi_0,\vec h,\vec x,t\right):=\\
\frac{1}{8\pi}\left|-\nabla_{\vec
x}\Psi-\frac{1}{c}\frac{\partial\vec A}{\partial t}+\frac{1}{c}\vec
v\times curl_{\vec x}\vec A\right|^2-\frac{1}{8\pi}\left|curl_{\vec
x}\vec A\right|^2-\left(\rho\Psi-\frac{1}{c}\vec A\cdot\vec j\right)
+\mu g\left(\frac{1}{2}\left|\vec u-\vec v\right|^2\right)
\\
-\frac{c^2}{8\pi G}\left|-\nabla_{\vec x}\Phi_0
-\frac{1}{c}\frac{\partial\vec h}{\partial t}
+\frac{1}{c}\vec v\times curl_{\vec x}\vec h-\frac{1}{c}\nabla_{\vec
x}\left(\vec h\cdot\vec v\right)\right|^2+\frac{c^2}{8\pi
G}\left|curl_{\vec x}\vec h\right|^2\\+\frac{c^2\beta}{4\pi
G}\left(\frac{1}{4}\left|d_{\vec x}\vec v+\left\{d_{\vec x}\vec
v\right\}^T\right|^2-\left|\Div_{\vec x}\vec v\right|^2\right) .
\end{multline}
In other words we have:
\begin{multline}\label{vhfffngghkjgghDDnw22jkjkjkjk22gughhg22mm34}
L\left(\vec A,\Psi,\vec v,\Phi_0,\vec h,\vec
x,t\right)=L_1\left(\vec A,\Psi,\vec v,\vec
x,t\right):=\\
\frac{1}{8\pi}\left|-\nabla_{\vec
x}\Psi-\frac{1}{c}\frac{\partial\vec A}{\partial t}+\frac{1}{c}\vec
v\times curl_{\vec x}\vec A\right|^2-\frac{1}{8\pi}\left|curl_{\vec
x}\vec A\right|^2-\left(\rho\Psi-\frac{1}{c}\vec A\cdot\vec
j\right)+\mu g\left(\frac{1}{2}\left|\vec u-\vec v\right|^2\right)
\\
-\frac{c^2}{8\pi G}\left|
\frac{1}{c}\nabla_{\vec x}\left(\frac{1}{2}\left|\vec v-\vec
k\right|^2\right)-\frac{1}{c}\frac{\partial}{\partial t}\left(\vec
v-\vec k\right)
+\frac{1}{c}\vec v\times curl_{\vec x}\left(\vec v-\vec
k\right)-\frac{1}{c}\nabla_{\vec x}\left(\left(\vec v-\vec
k\right)\cdot\vec v\right)\right|^2\\+\frac{c^2}{8\pi
G}\left|curl_{\vec x}\left(\vec v-\vec
k\right)\right|^2+\frac{c^2\beta}{4\pi
G}\left(\frac{1}{4}\left|d_{\vec x}\vec v+\left\{d_{\vec x}\vec
v\right\}^T\right|^2-\left|\Div_{\vec x}\vec v\right|^2\right).
\end{multline}
In particular, in the inertial coordinate system where $d_{\vec
x}\vec k=0$ and $\partial_t\vec k=0$ we have:
\begin{multline}\label{vhfffngghkjgghDDnw22jkjkjkjk22gughhg22hkhjhjhjh22kljlkjk22mma34}
L\left(\vec A,\Psi,\vec v,\Phi_0,\vec h,\vec
x,t\right)=L_1\left(\vec A,\Psi,\vec v,\vec
x,t\right):=\\
\frac{1}{8\pi}\left|-\nabla_{\vec
x}\Psi-\frac{1}{c}\frac{\partial\vec A}{\partial t}+\frac{1}{c}\vec
v\times curl_{\vec x}\vec A\right|^2-\frac{1}{8\pi}\left|curl_{\vec
x}\vec A\right|^2-\left(\rho\Psi-\frac{1}{c}\vec A\cdot\vec
j\right)+\mu g\left(\frac{1}{2}\left|\vec u-\vec v\right|^2\right)
\\
-\frac{c^2}{8\pi G}\left|
-\frac{1}{c}\frac{\partial\vec v}{\partial t}
+\frac{1}{c}\vec v\times curl_{\vec x}\vec v-\frac{1}{c}\nabla_{\vec
x}\left(\frac{1}{2}\left|\vec
v\right|^2\right)\right|^2+\frac{c^2}{8\pi G}\left|curl_{\vec x}\vec
v\right|^2\\+\frac{c^2\beta}{4\pi G}\left(\frac{1}{4}\left|d_{\vec
x}\vec v+\left\{d_{\vec x}\vec v\right\}^T\right|^2-\left|\Div_{\vec
x}\vec v\right|^2\right),
\end{multline}
Note here the advantage of inertial coordinate systems, where $L$
and $L_1$ are complectly independent of the vectorial potential of
the inertia $\vec k$. Furthermore, by \er{apfrm9} we rewrite
\er{vhfffngghkjgghDDnw22jkjkjkjk22gughhg22hkhjhjhjh22kljlkjk22mma34}
as:
\begin{multline}\label{vhfffngghkjgghDDnw22jkjkjkjk22gughhg22hkhjhjhjh22kljlkjk22mm34}
L\left(\vec A,\Psi,\vec v,\Phi_0,\vec h,\vec
x,t\right)=L_1\left(\vec A,\Psi,\vec v,\vec
x,t\right):=\\
\frac{1}{8\pi}\left|-\nabla_{\vec
x}\Psi-\frac{1}{c}\frac{\partial\vec A}{\partial t}+\frac{1}{c}\vec
v\times curl_{\vec x}\vec A\right|^2-\frac{1}{8\pi}\left|curl_{\vec
x}\vec A\right|^2-\left(\rho\Psi-\frac{1}{c}\vec A\cdot\vec
j\right)+\mu g\left(\frac{1}{2}\left|\vec u-\vec v\right|^2\right)
\\
-\frac{c^2}{8\pi G}\left|
-\frac{1}{c}\left(\frac{\partial\vec v}{\partial t}+d_{\vec x}\vec
v\cdot\vec v\right)\right|^2+\frac{c^2}{8\pi G}\left|curl_{\vec
x}\vec v\right|^2+\frac{c^2\beta}{4\pi
G}\left(\frac{1}{4}\left|d_{\vec x}\vec v+\left\{d_{\vec x}\vec
v\right\}^T\right|^2-\left|\Div_{\vec x}\vec v\right|^2\right).
\end{multline}
Then, using Proposition \ref{yghgjtgyrtrt} by
\er{vhfffngghkjgghDDnw22mm34}, \er{btfffygtgyggyDDnw22jkhk22mm34}
and \er{vhfffngghkjgghDDnw22jkjkjkjk22gughhg22mm34} we deduce that
$L$ and $L_1$ are invariant under the change of inertial or
non-inertial cartesian coordinate system, provided that $\vec h$ and
$\vec A$ are proper vector fields, $\vec v$ is a speed-like vector
field and $\Phi_0$ and $\Psi_0:=\Psi-\frac{1}{c}\vec A\cdot\vec v$
are proper scalar fields.

We point out, again, two the most important choices of function
$\mathcal{G}(s)$ in \er{vhfffngghkjgghDDnw22mm34mm34}: fully
non-relativistic choice $\mathcal{G}(s)=\frac{s}{2}$ and
correspondingly $g(s)=\left(s-\frac{c^2}{2}\right)$; and
relativistic-like choice $\mathcal{G}(s)=\sqrt{s}$ and
correspondingly $g(s):=-c^2\sqrt{1-\frac{2s}{c^2}}$. Note also that
in the first case we have $\frac{dg}{ds}(s)=1$ and in the second
case
$\frac{dg}{ds}(s)=\left(1-\frac{2s}{c^2}\right)^{-\frac{1}{2}}\approx
1$, where the last equation is valid in the case where $2s\ll c^2$.

We investigate critical points of the functional
\begin{equation}\label{btfffygtgyggyDDnw22mm34}
J=\int_0^T\int_{\mathbb{R}^3}L_1\left(\vec A,\Psi,\vec v,\vec
x,t\right)d\vec x dt.
\end{equation}
We denote
\begin{equation}\label{guigjgjffghDDnw22mm34}
\begin{cases}
\Psi_0=\Psi-\frac{1}{c}\vec A\cdot\vec v\\
\vec D=-\nabla_{\vec x}\Psi-\frac{1}{c}\frac{\partial\vec
A}{\partial t}+\frac{1}{c}\vec v\times curl_{\vec x}\vec
A=-\nabla_{\vec x}\Psi_0-\frac{1}{c}\frac{\partial\vec A}{\partial
t}+\frac{1}{c}\vec
v\times curl_{\vec x}\vec A-\frac{1}{c}\nabla_{\vec x}\left(\vec A\cdot\vec v\right)\\
\vec B=curl_{\vec x}\vec A
\\
\vec E=-\nabla_{\vec x}\Psi-\frac{1}{c}\frac{\partial\vec A}{\partial t}=\vec D-\frac{1}{c}\vec v\times\vec B\\
\vec H=curl_{\vec x}\vec A+\frac{1}{c}\vec
v\times\left(-\nabla_{\vec x}\Psi-\frac{1}{c}\frac{\partial\vec
A}{\partial t}+\frac{1}{c}\vec v\times curl_{\vec x}\vec
A\right)=\vec B+\frac{1}{c}\vec v\times\vec D.
\end{cases}
\end{equation}
and
\begin{equation}\label{guigjgjffghDDnwjkhhkhkhjklh22mm34}
\begin{cases}
\vec R=-\nabla_{\vec x}\Phi_0-\frac{1}{c}\frac{\partial\vec
h}{\partial t}
+\frac{1}{c}\vec v\times curl_{\vec x}\vec
h-\frac{1}{c}\nabla_{\vec x}\left(\vec h\cdot\vec v\right)\\
\vec Q=curl_{\vec x}\vec h,
\end{cases}
\end{equation}
where $\vec h$ is a proper vector field and $\Phi_0$ is a proper
scalar field that are given by \er{btfffygtgyggyDDnw22jkhk22mm34}.
In other words,
\begin{equation}\label{guigjgjffghDDnwjkhhkhkhjklh22ljjljkhk22mm34}
\begin{cases}
\vec R=\frac{1}{c}\nabla_{\vec x}\left(\frac{1}{2}\left|\vec v-\vec
k\right|^2\right)-\frac{1}{c}\frac{\partial}{\partial t}\left(\vec
v-\vec k\right)
+\frac{1}{c}\vec v\times curl_{\vec x}\left(\vec v-\vec
k\right)-\frac{1}{c}\nabla_{\vec x}\left(\left(\vec v-\vec
k\right)\cdot\vec v\right)\\
\vec Q=curl_{\vec x}(\vec v-\vec k),
\end{cases}
\end{equation}
and in inertial coordinate system where $d_{\vec x}\vec k=0$ and
$\partial_t\vec k=0$ we also have:
\begin{equation}\label{guigjgjffghDDnwjkhhkhkhjklh22ljjljkhk22hhmm34}
\begin{cases}
\vec R=-\frac{1}{c}\left(\frac{\partial\vec v}{\partial t}+d_{\vec
x}\vec
v\cdot\vec v\right)\\
\vec Q=curl_{\vec x}\vec v.
\end{cases}
\end{equation}
As before, by \er{guigjgjffghDDnw22mm34} and
\er{guigjgjffghDDnwjkhhkhkhjklh22mm34} and Proposition
\ref{yghgjtgyrtrt} we infer that both $\vec D$, $\vec B$ and $\vec
R$, $\vec Q$ are proper vector fields. Next, by
\er{guigjgjffghDDnw22mm34} we have:
\begin{equation}\label{guigjgjffghjhkkgDDnw22mmff34}
\begin{cases}
curl_{\vec x}\vec D+\frac{1}{c}\frac{\partial\vec B}{\partial t}
-\frac{1}{c}\,curl_{\vec x}\left(\vec v\times\vec B\right)=curl_{\vec x}\vec E+\frac{1}{c}\frac{\partial\vec B}{\partial t}=0\\
div_{\vec x}\vec B=0,
\end{cases}
\end{equation}
and by \er{guigjgjffghDDnwjkhhkhkhjklh22mm34} we have:
\begin{equation}\label{guigjgjffghjhkkgDDnw22mm34}
\begin{cases}
curl_{\vec x}\vec R+\frac{1}{c}\frac{\partial\vec Q}{\partial t}-\frac{1}{c}\,curl_{\vec x}\left(\vec v\times\vec Q\right)=0\\
div_{\vec x}\vec Q=0.
\end{cases}
\end{equation}
Furthermore, by \er{guigjgjffghDDnwjkhhkhkhjklh22ljjljkhk22mm34} and
\er{apfrm9} we deduce
\begin{multline}\label{guigjgjffghDDnwjkhhkhkhjklh22ljjljkhk22hkjhjhjh22mm34}
\vec R=\frac{1}{c}\nabla_{\vec x}\left(\frac{1}{2}\left|\vec v-\vec
k\right|^2\right)-\frac{1}{c}\left\{d_{\vec x}\vec
v\right\}^T\cdot\left(\vec v-\vec
k\right)-\frac{1}{c}\frac{\partial}{\partial t}\left(\vec v-\vec
k\right)-\frac{1}{c}\left\{d_{\vec x}\left(\vec v-\vec
k\right)\right\}^T\cdot\vec v
+\frac{1}{c}\vec v\times curl_{\vec x}\left(\vec v-\vec k\right)\\=
-\frac{1}{c}\left\{d_{\vec x}\vec k\right\}^T\cdot\left(\vec v-\vec
k\right)-\frac{1}{c}\frac{\partial}{\partial t}\left(\vec v-\vec
k\right)-\frac{1}{c}d_{\vec x}\left(\vec v-\vec k\right)\cdot\vec v
\\
\quad\text{and}\quad \vec Q=curl_{\vec x}(\vec v-\vec k),
\end{multline}
and thus, since $d_{\vec x}\vec k+\{d_{\vec x}\vec k\}^T=0$,
$curl_{\vec x}(curl_{\vec x}\vec k)=0$ and $\Div_{\vec x}\vec k=0$
we infer
\begin{multline}\label{guigjgjffghDDnwjkhhkhkhjklh22ljjljkhk22hkjhjhjh22jljjkhjgh22hhjhmm34}
\Div_{\vec x}\vec R=
-\frac{1}{c}\left(\vec v-\vec k\right)\cdot\Delta_{\vec x}\vec
k-\frac{1}{c}tr\left(\{d_{\vec x}\vec k\}^T\cdot d_{\vec
x}\left(\vec v-\vec
k\right)\right)-\frac{1}{c}\frac{\partial}{\partial
t}\left(\Div_{\vec x}\vec v\right)-\frac{1}{c}d_{\vec
x}\left(\Div_{\vec x}\vec v\right)\cdot\vec
v\\-\frac{1}{c}tr\left(d_{\vec x}\left(\vec v-\vec k\right)\cdot
d_{\vec x}\vec v
\right)\\
=-\frac{1}{c}\frac{\partial}{\partial t}\left(\Div_{\vec x}\vec
v\right)-\frac{1}{c}d_{\vec x}\left(\Div_{\vec x}\vec
v\right)\cdot\vec v-\frac{1}{c}tr\left(d_{\vec x}\left(\vec v-\vec
k\right)\cdot
d_{\vec x}(\vec v-\vec k)
\right)\\
=-\frac{1}{c}\frac{\partial}{\partial t}\left(\Div_{\vec x}\vec
v\right)-\frac{1}{c}d_{\vec x}\left(\Div_{\vec x}\vec
v\right)\cdot\vec v-\frac{1}{4c}\left|d_{\vec x}(\vec v-\vec
k)+\left\{d_{\vec x}(\vec v-\vec
k)\right\}^T\right|^2+\frac{1}{4c}\left|d_{\vec x}(\vec v-\vec
k)-\left\{d_{\vec x}(\vec v-\vec k)\right\}^T\right|^2
\\
=-\frac{1}{c}\frac{\partial}{\partial t}\left(\Div_{\vec x}\vec
v\right)-\frac{1}{c}d_{\vec x}\left(\Div_{\vec x}\vec
v\right)\cdot\vec v-\frac{1}{4c}\left|d_{\vec x}\vec
v+\left\{d_{\vec x}\vec v\right\}^T\right|^2+\frac{1}{2c}\left|\vec
Q\right|^2
\\
\quad\text{and}\quad curl_{\vec x}\vec Q=curl_{\vec x}(curl_{\vec
x}\vec v).
\end{multline}
So,
\begin{multline}\label{guigjgjffghDDnwjkhhkhkhjklh22ljjljkhk22hkjhjhjh22jljjkhjgh22hhjhjijhjhjtmm34}
\frac{1}{c}\left(\frac{\partial}{\partial t}\left(\Div_{\vec x}\vec
v\right)+d_{\vec x}\left(\Div_{\vec x}\vec v\right)\cdot\vec
v+\frac{1}{4}\left|d_{\vec x}\vec v+\left\{d_{\vec x}\vec
v\right\}^T\right|^2\right)=\frac{1}{2c}\left|\vec
Q\right|^2-\Div_{\vec x}\vec R
\\
\quad\text{and}\quad curl_{\vec x}(curl_{\vec x}\vec v)=curl_{\vec
x}\vec Q,
\end{multline}
In other words,
\begin{multline}\label{guigjgjffghDDnwjkhhkhkhjklh22ljjljkhk22hkjhjhjh22jljjkhjgh22hhjhjijhjhjtjljjmm34}
\frac{1}{c}\left(\frac{\partial}{\partial t}\left(\Div_{\vec x}\vec
v\right)+\Div_{\vec x}\left\{\left(\Div_{\vec x}\vec v\right)\vec
v\right\}+\frac{1}{4}\left|d_{\vec x}\vec v+\left\{d_{\vec x}\vec
v\right\}^T\right|^2-\left|\Div_{\vec x}\vec
v\right|^2\right)=\frac{1}{2c}\left|\vec Q\right|^2-\Div_{\vec
x}\vec R
\\
\quad\text{and}\quad curl_{\vec x}(curl_{\vec x}\vec v)=curl_{\vec
x}\vec Q.
\end{multline}
Moreover, by \er{vhfffngghkjgghDDnw22mm34}, and  \er{apfrm3}
we have
\begin{equation}\label{vhfffngghkjgghggtghjgfhjyuiutDDnw22mm34}
\frac{\delta L}{\delta \vec h}=
\frac{c^2}{4\pi G}curl_{\vec x}\vec Q
-\frac{c}{4\pi G}\left(\frac{\partial\vec R}{\partial t}-curl_{\vec
x}\left\{\vec v\times\vec R\right\}+\left(div_{\vec x}\vec
R\right)\vec v\right),
\end{equation}
\begin{equation}\label{vhfffngghkjgghggtghjgfhjoyuiyuDDnw22mmhkkhjhmm34}
\frac{\delta L}{\delta{\Phi_0}}=-\frac{c^2}{4\pi G}\left(div_{\vec
x}\vec R\right).
\end{equation}
and
\begin{multline}\label{vhfffngghkjgghggtghjgfhjoyuiyuDDnw22mm34}
\frac{\delta L}{\delta \vec v}=-\left(\mu
g'\left(\frac{1}{2}\left|\vec u-\vec v\right|^2\right)(\vec u-\vec
v)+\frac{1}{4\pi c}\vec D\times\vec B-\frac{c}{4\pi G}\vec
R\times\vec Q\right)\\ +\frac{c^2\beta}{4\pi G}\,curl_{\vec
x}(curl_{\vec x}\vec v)-\frac{c}{4\pi G}\left(div_{\vec x}\vec
R\right)\vec h.
\end{multline}
Thus, since by \er{btfffygtgyggyDDnw22jkhk22mm34} we have
\begin{equation}\label{jkhlgytguoiyhioyhmm34}
L_1\left(\vec A,\Psi,\vec v,\vec x,t\right)=L\left(\vec A,\Psi,\vec
v,-\frac{1}{2c}\left|\vec v-\vec k\right|^2,(\vec v-\vec k),\vec
x,t\right)=L\left(\vec A,\Psi,\vec v,\Phi_0,\vec h,\vec x,t\right),
\end{equation}
by Chain rule we have
\begin{equation}\label{jkhlgytguoiyhioyh11mm34}
\frac{\delta L_1}{\delta \vec v}\left(\vec A,\Psi,\vec v,\vec
x,t\right)=
\frac{\delta L}{\delta \vec v}+\frac{\delta L}{\delta \vec
h}-\frac{1}{c}\frac{\delta L}{\delta {\Phi_0}}\left(\vec v-\vec
k\right)
\end{equation}
Therefore, using \er{vhfffngghkjgghggtghjgfhjyuiutDDnw22mm34},
\er{vhfffngghkjgghggtghjgfhjoyuiyuDDnw22mmhkkhjhmm34}
\er{vhfffngghkjgghggtghjgfhjoyuiyuDDnw22mm34} in
\er{jkhlgytguoiyhioyh11mm34} we deduce:
\begin{multline}\label{vhfffngghkjgghggtghjgfhjyuiutDDnw22kllkjkllkjl22klljjkjk22mm34}
\left(\mu g'\left(\frac{1}{2}\left|\vec u-\vec
v\right|^2\right)(\vec u-\vec v)+\frac{1}{4\pi c}\vec D\times\vec
B-\frac{c}{4\pi G}\vec R\times\vec Q\right)=\\ \frac{c^2\beta}{4\pi
G}\,curl_{\vec x}(curl_{\vec x}\vec v)+\frac{c^2}{4\pi G}curl_{\vec
x}\vec Q
-\frac{c}{4\pi G}\left(\frac{\partial\vec R}{\partial t}-curl_{\vec
x}\left\{\vec v\times\vec R\right\}+\left(div_{\vec x}\vec
R\right)\vec v\right).
\end{multline}
Moreover,
\begin{equation}\label{vhfffngghkjgghggtghjgfhjhjkghghDDnw22mm34}
\frac{\delta L}{\delta \Psi}=\frac{1}{4\pi}div_{\vec x}\vec
D-\rho=0,
\end{equation}
and
\begin{equation}\label{vhfffngghkjgghggtghjgfhjhjkghghyuiuuDDnw22mm34}
\frac{\delta L}{\delta \vec A}=\frac{1}{c}\vec j+\frac{1}{4\pi
c}\frac{\partial\vec D}{\partial t}-\frac{1}{4\pi}curl_{\vec x}\vec
B-\frac{1}{4\pi c}curl_{\vec x}\left(\vec v\times \vec
D\right)=\frac{1}{c}\vec j+\frac{1}{4\pi c}\frac{\partial\vec
D}{\partial t}-\frac{1}{4\pi}curl_{\vec x}\vec H=0.
\end{equation}
So
\begin{equation}\label{guigjgjffghguygjyfDDnw22mmd34}
\begin{cases}
curl_{\vec x}\vec B=\frac{4\pi}{c}\vec
j+\frac{1}{c}\frac{\partial\vec D}{\partial t}-\frac{1}{c}curl_{\vec
x}\left(\vec v\times \vec D\right)
\\
div_{\vec x}\vec D=4\pi\rho
\\
curl_{\vec x}\vec D+\frac{1}{c}\frac{\partial\vec B}{\partial t}
-\frac{1}{c}\,curl_{\vec x}\left(\vec v\times\vec B\right)=0\\
div_{\vec x}\vec B=0\\
curl_{\vec x}\vec R+\frac{1}{c}\frac{\partial\vec Q}{\partial t}-\frac{1}{c}curl_{\vec x}\left(\vec v\times\vec Q\right)=0\\
div_{\vec x}\vec Q=0
\\
\frac{1}{c}\left(\frac{\partial}{\partial t}\left(\Div_{\vec x}\vec
v\right)+\vec v\cdot\nabla_{\vec x}\left(\Div_{\vec x}\vec
v\right)+\frac{1}{4}\left|d_{\vec x}\vec v+\left\{d_{\vec x}\vec
v\right\}^T\right|^2\right)=\frac{1}{2c}\left|\vec
Q\right|^2-\Div_{\vec x}\vec R
\\
curl_{\vec x}(curl_{\vec x}\vec v)=curl_{\vec x}\vec Q
\\
\frac{4\pi G}{c^2}\left(\mu g'\left(\frac{1}{2}\left|\vec u-\vec
v\right|^2\right)(\vec u-\vec v)+\frac{1}{4\pi c}\vec D\times\vec
B-\frac{c}{4\pi G}\vec R\times\vec Q\right)=\\
(1+\beta)\,curl_{\vec x}\vec Q
-\frac{1}{c}\left(\frac{\partial\vec R}{\partial t}-curl_{\vec
x}\left\{\vec v\times\vec R\right\}+\left(div_{\vec x}\vec
R\right)\vec v\right),
\end{cases}
\end{equation}
and by \er{guigjgjffghDDnwjkhhkhkhjklh22ljjljkhk22hhmm34} in the
inertial frame we have:
\begin{equation}\label{guigjgjffghguygjyfDDnw22khhjhgg22mmd34}
\begin{cases}
curl_{\vec x}\vec B=\frac{4\pi}{c}\vec
j+\frac{1}{c}\frac{\partial\vec D}{\partial t}-\frac{1}{c}curl_{\vec
x}\left(\vec v\times \vec
D\right)\\
div_{\vec x}\vec D=4\pi\rho\\
curl_{\vec x}\vec D+\frac{1}{c}\frac{\partial\vec B}{\partial t}
-\frac{1}{c}\,curl_{\vec x}\left(\vec v\times\vec B\right)=0\\
div_{\vec x}\vec B=0\\
\vec R=-\frac{1}{c}\left(\frac{\partial\vec v}{\partial t}+d_{\vec
x}\vec
v\cdot\vec v\right)\\
\vec Q=curl_{\vec x}\vec v
\\
\frac{4\pi G}{c^2}\left(\mu g'\left(\frac{1}{2}\left|\vec u-\vec
v\right|^2\right)(\vec u-\vec v)+\frac{1}{4\pi c}\vec D\times\vec
B-\frac{c}{4\pi G}\vec R\times\vec Q\right)=\\
(1+\beta)\,curl_{\vec x}\vec Q
-\frac{1}{c}\left(\frac{\partial\vec R}{\partial t}-curl_{\vec
x}\left\{\vec v\times\vec R\right\}+\left(div_{\vec x}\vec
R\right)\vec v\right).
\end{cases}
\end{equation}
Furthermore,  taking $\Div_{\vec x}$ of the both sides of the last
equality in \er{guigjgjffghguygjyfDDnw22mmd34} and using continuum
equation $\partial_t\mu+\Div_{\vec x}\left(\mu\vec u\right)=0$ we
deduce
\begin{multline}\label{guigjgjffghguygjyfDDnw22mmz34}
-\left(\partial_t\mu+\Div_{\vec x}\left(\mu\vec
v\right)\right)+\Div_{\vec x}\left\{\mu
\left(g'\left(\frac{1}{2}\left|\vec u-\vec
v\right|^2\right)-1\right)(\vec u-\vec v)+\frac{1}{4\pi c}\vec
D\times\vec B-\frac{c}{4\pi  G}\vec R\times\vec Q\right\}=\\
\Div_{\vec x}\left(\mu g'\left(\frac{1}{2}\left|\vec u-\vec
v\right|^2\right)(\vec u-\vec v)+\frac{1}{4\pi c}\vec D\times\vec
B-\frac{c}{4\pi  G}\vec R\times\vec Q\right)=\\
-\frac{c}{4\pi  G}\left(\frac{\partial}{\partial t}\left(\Div_{\vec
x}\vec R\right)+\Div_{\vec x}\left\{\left(div_{\vec x}\vec
R\right)\vec v\right\}\right),
\end{multline}
Therefore, considering the proper scalar quantity $Q_0$, that we
call the field mass, which satisfies
\begin{equation}\label{jtgiktiutiurtfirfmmhfffmm34}
Q_0:=-\mu+\frac{c}{4\pi  G}\,\Div_{\vec x}\vec R,
\end{equation}
by \er{guigjgjffghguygjyfDDnw22mmz34} we deduce
\begin{equation}\label{jtgiktiutiurtfirfmm34}
\frac{\partial Q_0}{\partial t}+div_{\vec x}\left(Q_0\vec
v\right)=-\Div_{\vec x}\left\{\mu
\left(g'\left(\frac{1}{2}\left|\vec u-\vec
v\right|^2\right)-1\right)(\vec u-\vec v)+\frac{1}{4\pi c}\vec
D\times\vec B-\frac{c}{4\pi  G}\vec R\times\vec Q\right\}.
\end{equation}
Thus, we rewrite \er{guigjgjffghguygjyfDDnw22mmd34} as:
\begin{equation}\label{guigjgjffghguygjyfDDnw22mm34}
\begin{cases}
curl_{\vec x}\vec B=\frac{4\pi}{c}\left(\vec j-\rho\vec
v\right)+\frac{1}{c}\frac{\partial\vec D}{\partial
t}-\frac{1}{c}curl_{\vec x}\left(\vec v\times \vec
D\right)+\frac{1}{c}\left(div_{\vec x}\vec D\right)\vec v,
\\
div_{\vec x}\vec D=4\pi\rho,
\\
curl_{\vec x}\vec D+\frac{1}{c}\frac{\partial\vec B}{\partial t}
-\frac{1}{c}\,curl_{\vec x}\left(\vec v\times\vec B\right)=0,\\
div_{\vec x}\vec B=0,\\
curl_{\vec x}\vec R+\frac{1}{c}\frac{\partial\vec Q}{\partial t}-\frac{1}{c}curl_{\vec x}\left(\vec v\times\vec Q\right)=0,\\
div_{\vec x}\vec Q=0,
\\
\frac{4\pi G}{c^2}\left(\mu g'\left(\frac{1}{2}\left|\vec u-\vec
v\right|^2\right)(\vec u-\vec v)+\frac{1}{4\pi c}\vec D\times\vec
B-\frac{c}{4\pi G}\vec R\times\vec Q\right)=\\
(1+\beta)\,curl_{\vec x}\vec Q
-\frac{1}{c}\left(\frac{\partial\vec R}{\partial t}-curl_{\vec
x}\left\{\vec v\times\vec R\right\}+\left(div_{\vec x}\vec
R\right)\vec v\right),\\
\Div_{\vec x}\vec R=\frac{4\pi G}{c}(\mu+Q_0),
\\
\frac{1}{c}\left(\frac{\partial}{\partial t}\left(\Div_{\vec x}\vec
v\right)+\vec v\cdot\nabla_{\vec x}\left(\Div_{\vec x}\vec
v\right)+\frac{1}{4}\left|d_{\vec x}\vec v+\left\{d_{\vec x}\vec
v\right\}^T\right|^2\right)=\frac{1}{2c}\left|\vec
Q\right|^2-\Div_{\vec x}\vec R,
\\
curl_{\vec x}(curl_{\vec x}\vec v)=curl_{\vec x}\vec Q,
\\
\frac{\partial Q_0}{\partial t}+div_{\vec x}\left(Q_0\vec
v\right)=-\Div_{\vec x}\left\{\mu
\left(g'\left(\frac{1}{2}\left|\vec u-\vec
v\right|^2\right)-1\right)(\vec u-\vec v)+\frac{1}{4\pi c}\vec
D\times\vec B-\frac{c}{4\pi  G}\vec R\times\vec Q\right\},
\end{cases}
\end{equation}
and we rewrite \er{guigjgjffghguygjyfDDnw22khhjhgg22mmd34} in the
inertial frame as:
\begin{equation}\label{guigjgjffghguygjyfDDnw22khhjhgg22mm34}
\begin{cases}
curl_{\vec x}\vec B=\frac{4\pi}{c}\vec
j+\frac{1}{c}\frac{\partial\vec D}{\partial t}-\frac{1}{c}curl_{\vec
x}\left(\vec v\times \vec
D\right),\\
div_{\vec x}\vec D=4\pi\rho,\\
curl_{\vec x}\vec D+\frac{1}{c}\frac{\partial\vec B}{\partial t}
-\frac{1}{c}\,curl_{\vec x}\left(\vec v\times\vec B\right)=0,\\
div_{\vec x}\vec B=0,\\
\vec R=-\frac{1}{c}\left(\frac{\partial\vec v}{\partial t}+d_{\vec
x}\vec
v\cdot\vec v\right),\\
\vec Q=curl_{\vec x}\vec v,
\\
\frac{4\pi G}{c^2}\left(\mu g'\left(\frac{1}{2}\left|\vec u-\vec
v\right|^2\right)(\vec u-\vec v)+\frac{1}{4\pi c}\vec D\times\vec
B-\frac{c}{4\pi G}\vec R\times\vec Q\right)=\\
(1+\beta)\,curl_{\vec x}\vec Q
-\frac{1}{c}\left(\frac{\partial\vec R}{\partial t}-curl_{\vec
x}\left\{\vec v\times\vec R\right\}+\left(div_{\vec x}\vec
R\right)\vec v\right),\\
\Div_{\vec x}\vec R=\frac{4\pi G}{c} (\mu+Q_0),\\
\frac{\partial Q_0}{\partial t}+div_{\vec x}\left(Q_0\vec
v\right)=-\Div_{\vec x}\left\{\mu
\left(g'\left(\frac{1}{2}\left|\vec u-\vec
v\right|^2\right)-1\right)(\vec u-\vec v)+\frac{1}{4\pi c}\vec
D\times\vec B-\frac{c}{4\pi  G}\vec R\times\vec Q\right\}.
\end{cases}
\end{equation}
As before, by Proposition \ref{yghgjtgyrtrt} we deduce that
\er{guigjgjffghguygjyfDDnw22mm34} is invariant under the change of
inertial or non-inertial cartesian coordinate systems. Moreover,
\er{guigjgjffghguygjyfDDnw22khhjhgg22mm34} is invariant under the
change of inertial cartesian coordinate systems. Note also that,
both, \er{guigjgjffghguygjyfDDnw22mm34} in an arbitrary inertial or
non-inertial cartesian coordinate system and
\er{guigjgjffghguygjyfDDnw22khhjhgg22mm34} in an arbitrary inertial
cartesian coordinate system, are complectly independent of the
vectorial potential of the inertia $\vec k$. Note again for the last
equation in \er{guigjgjffghguygjyfDDnw22mm34} or
\er{guigjgjffghguygjyfDDnw22khhjhgg22mm34} that: in the fully
non-relativistic case we have $g'(s)=1$ and in the relativistic-like
case we have
$g'(s)=\left(1-\frac{2s}{c^2}\right)^{-\frac{1}{2}}\approx 1$, where
the last equation is valid in the case where $2s\ll c^2$.

Finally, as before, note that in the case of large constant
$|\beta|\gg 1$ we have $\vec Q\to 0$ in
\er{guigjgjffghguygjyfDDnw22mm34} and thus, the gravity equations
\er{guigjgjffghguygjyfDDnw22mm34} reduce to the equations of the
Newtonian-type Gravity in the form of
\er{MaxVacFull1bjkgjhjhgjaaahkjhhjzzzyyykkknnnNWBWHWPPNNewCCmmhhjjhkkk}.
In that case the gravity field propagates with the infinite speed.
On the other hand, in the case of vanishing constant $\beta=0$ the
form of equations for $\vec R$ and $\vec Q$ in
\er{guigjgjffghguygjyfDDnw22mm34} is completely the same as the form
of the Maxwell equations for $\vec D$ and $\vec B$ in
\er{guigjgjffghguygjyfDDnw22mm34}, except of the different meaning
of "charges" and "currents" in these two sets of equations. In that
case the electromagnetic and the gravity fields propagate with the
same speed. However, in the mixed case of constant $\beta\sim 1$ the
electromagnetic and the gravity fields propagate with different
finite speeds.

\begin{remark}\label{gygygygyggjh3499}
One can wonder: what should be possible values of the vectorial
gravitational potential $\vec v$ in the proximity of the Earth or
another massive body? In remark \ref{gygygygyggjh} we answered this
question in the case of the Newtonian-type gravity, given by
\er{MaxVacFull1ninshtrgravortghhghgjkgghklhjgkghghjjkjhjkkggjkhjkhjjhhfhjhklkhkhjjklzzzyyyNWBWHWPPN222ll}.
Moreover, in  remark \ref{gygygygyggjh34} we answered this question
in the case of more general laws of the gravity, given by
\er{guigjgjffghguygjyfDDnw22mm} with arbitrary constant $\beta$. In
order to answer this question in the case of more general laws of
the gravity, given by ether
\er{MaxVacFull1bjkgjhjhgjaaahkjhhjzzzyyykkknnnNWBWHWPPNNewCCmmhhjjhkkk}
or \er{guigjgjffghguygjyfDDnw22mm34} with arbitrary constant
$\beta$, consider two cartesian coordinate systems:
\underline{non-rotating} system $\bf(*)$ with the center that
coincides with the center of masses of the Earth and
\underline{inertial} system $\bf(**)$ related to some external
cosmic bodies. Assume that the center of masses of the Earth has
place $\vec R(t')$ and velocity $\vec W(t'):=\frac{d\vec
R}{dt'}(t')$ in the coordinate system $\bf(**)$. Thus the change of
coordinate system $\bf(*)$ to coordinate system $\bf(**)$ is given
by
\begin{equation}\label{hjjgghghghyu3499}
\begin{cases}
\vec x'=\vec x+\vec R(t),\\
t'=t,
\end{cases}
\end{equation}
and the vectorial gravitational potential $\vec v$, being a speed
like vector field, transforms as
\begin{equation}\label{hjjgghghghyulll3499}
\vec v'=\vec v+\vec W(t).
\end{equation}
Next, since the system $\bf(**)$ is inertial, consistently with
\er{guigjgjffghguygjyfDDnw22khhjhgg22mm34} in the system $\bf(**)$
we have
\begin{equation}\label{guigjgjffghguygjyfDDnw22khhjhgg22mm34hhhhhh3499}
\begin{cases}
curl_{\vec x'}\vec B'=\frac{4\pi}{c}\left(\vec j'-\rho'\vec
v'\right)+\frac{1}{c}\frac{\partial\vec D'}{\partial
t'}-\frac{1}{c}curl_{\vec x'}\left(\vec v'\times \vec
D'\right)+\frac{1}{c}\left(div_{\vec x'}\vec D'\right)\vec v',
\\
div_{\vec x}\vec D'=4\pi\rho',
\\
curl_{\vec x'}\vec D'+\frac{1}{c}\frac{\partial\vec B'}{\partial t'}
-\frac{1}{c}\,curl_{\vec x'}\left(\vec v'\times\vec B'\right)=0,\\
div_{\vec x'}\vec B'=0,\\
\vec R'=-\frac{1}{c}\left(\frac{\partial\vec v'}{\partial
t'}+d_{\vec x'}\vec
v'\cdot\vec v'\right),\\
\vec Q'=curl_{\vec x'}\vec v',
\\
curl_{\vec x'}\vec R'+\frac{1}{c}\frac{\partial\vec Q'}{\partial t'}-\frac{1}{c}curl_{\vec x'}\left(\vec v'\times\vec Q'\right)=0,\\
div_{\vec x'}\vec Q'=0,
\\
\frac{4\pi G}{c^2}\left(\mu'g'\left(\frac{1}{2}\left|\vec u'-\vec
v'\right|^2\right)\left(\vec u'-\vec v'\right)+\frac{1}{4\pi c}\vec
D'\times\vec
B'-\frac{c}{4\pi G}\vec R'\times\vec Q'\right)=\\
(1+\beta)\,curl_{\vec x'}\vec Q'
-\frac{1}{c}\left(\frac{\partial\vec R'}{\partial t'}-curl_{\vec
x'}\left\{\vec v'\times\vec R'\right\}+\left(div_{\vec x}\vec
R'\right)\vec v'\right),\\
\Div_{\vec x'}\vec R'=\frac{4\pi G}{c}(\mu'+Q'_0),
\\
\frac{\partial Q'_0}{\partial t'}+div_{\vec x'}\left(Q'_0\vec
v'\right)=\\-\Div_{\vec x'}\left\{
\mu'\left(g'\left(\frac{1}{2}\left|\vec u'-\vec
v'\right|^2\right)-1\right)\left(\vec u'-\vec v'\right)+
\frac{1}{4\pi c}\vec D'\times\vec B'-\frac{c}{4\pi  G}\vec
R'\times\vec Q'\right\},
\end{cases}
\end{equation}
and by \er{guigjgjffghguygjyfDDnw22mm34} in the system $\bf(*)$ we
have
\begin{equation}\label{guigjgjffghguygjyfDDnw22mm343499}
\begin{cases}
curl_{\vec x}\vec B=\frac{4\pi}{c}\left(\vec j-\rho\vec
v\right)+\frac{1}{c}\frac{\partial\vec D}{\partial
t}-\frac{1}{c}curl_{\vec x}\left(\vec v\times \vec
D\right)+\frac{1}{c}\left(div_{\vec x}\vec D\right)\vec v,
\\
div_{\vec x}\vec D=4\pi\rho,
\\
curl_{\vec x}\vec D+\frac{1}{c}\frac{\partial\vec B}{\partial t}
-\frac{1}{c}\,curl_{\vec x}\left(\vec v\times\vec B\right)=0,\\
div_{\vec x}\vec B=0,\\
curl_{\vec x}\vec R+\frac{1}{c}\frac{\partial\vec Q}{\partial t}-\frac{1}{c}curl_{\vec x}\left(\vec v\times\vec Q\right)=0,\\
div_{\vec x}\vec Q=0,
\\
\frac{4\pi G}{c^2}\left(\mu g'\left(\frac{1}{2}\left|\vec u-\vec
v\right|^2\right)\left(\vec u-\vec v\right)+\frac{1}{4\pi c}\vec
D\times\vec
B-\frac{c}{4\pi G}\vec R\times\vec Q\right)=\\
(1+\beta)\,curl_{\vec x}\vec Q
-\frac{1}{c}\left(\frac{\partial\vec R}{\partial t}-curl_{\vec
x}\left\{\vec v\times\vec R\right\}+\left(div_{\vec x}\vec
R\right)\vec v\right),\\
\Div_{\vec x}\vec R=\frac{4\pi G}{c}(\mu+Q_0),
\\
\frac{1}{c}\left(\frac{\partial}{\partial t}\left(\Div_{\vec x}\vec
v\right)+\vec v\cdot\nabla_{\vec x}\left(\Div_{\vec x}\vec
v\right)+\frac{1}{4}\left|d_{\vec x}\vec v+\left\{d_{\vec x}\vec
v\right\}^T\right|^2\right)=\frac{1}{2c}\left|\vec
Q\right|^2-\Div_{\vec x}\vec R,
\\
curl_{\vec x}(curl_{\vec x}\vec v)=curl_{\vec x}\vec Q,
\\
\frac{\partial Q_0}{\partial t}+div_{\vec x}\left(Q_0\vec
v\right)=-\Div_{\vec x}\left\{\mu
\left(g'\left(\frac{1}{2}\left|\vec u-\vec
v\right|^2\right)-1\right)(\vec u-\vec v)+\frac{1}{4\pi c}\vec
D\times\vec B-\frac{c}{4\pi  G}\vec R\times\vec Q\right\}.
\end{cases}
\end{equation}
On the other hand, by \er{hjjgghghghyu3499} and
\er{hjjgghghghyulll3499}
and using Proposition \ref{yghgjtgyrtrt}
we deduce
\begin{equation}
\label{MaxVacFull1ninshtrgravortghhghgjkgghklhjgkghghjjkjhjkkggjkhjkhjjhhfhjhklkhkhjjklzzzyyyhjggjhgghhjhNWNWBWHWPPN222kkkgtytghjjh3434499}
curl_{\vec x}\vec v=curl_{\vec x'}\vec v',
\end{equation}
and
\begin{equation}
\label{MaxVacFull1ninshtrgravortghhghgjkgghklhjgkghghjjkjhjkkggjkhjkhjjhhfhjhklkhkhjjklzzzyyyhjggjhgghhjhNWNWBWHWPPN222kkkgtytghjjhghhg3434499}
\frac{d\vec W}{\partial t}(t)+\frac{\partial\vec v}{\partial
t}+d_{\vec x}\vec v\cdot\vec v=\frac{d\vec W}{\partial
t'}(t')+\frac{\partial\vec v}{\partial t'}+d_{\vec x'}\vec
v\cdot\vec W(t')+d_{\vec x'}\vec v\cdot\vec v= \frac{\partial\vec
v'}{\partial t'}+d_{\vec x'}\vec v'\cdot\vec v'.
\end{equation}
Thus since $\vec R$ is a proper field, by
\er{guigjgjffghguygjyfDDnw22khhjhgg22mm34hhhhhh3499},
\er{MaxVacFull1ninshtrgravortghhghgjkgghklhjgkghghjjkjhjkkggjkhjkhjjhhfhjhklkhkhjjklzzzyyyhjggjhgghhjhNWNWBWHWPPN222kkkgtytghjjh3434499}
and
\er{MaxVacFull1ninshtrgravortghhghgjkgghklhjgkghghjjkjhjkkggjkhjkhjjhhfhjhklkhkhjjklzzzyyyhjggjhgghhjhNWNWBWHWPPN222kkkgtytghjjhghhg3434499}
we deduce
\begin{equation}
\label{MaxVacFull1ninshtrgravortghhghgjkgghklhjgkghghjjkjhjkkggjkhjkhjjhhfhjhklkhkhjjklzzzyyyhjggjhgghhjhNWNWBWHWPPN222kkkgtytghjjh343445499}
curl_{\vec x}\vec v=\vec Q'=\vec Q,
\end{equation}
and
\begin{equation}
\label{MaxVacFull1ninshtrgravortghhghgjkgghklhjgkghghjjkjhjkkggjkhjkhjjhhfhjhklkhkhjjklzzzyyyhjggjhgghhjhNWNWBWHWPPN222kkkgtytghjjhghhg343445499}
\frac{d\vec W}{\partial t}(t)+\frac{\partial\vec v}{\partial
t}+d_{\vec x}\vec v\cdot\vec v=-c\,\vec R'=-c\,\vec R.
\end{equation}
Therefore, by \er{guigjgjffghguygjyfDDnw22mm343499},
\er{MaxVacFull1ninshtrgravortghhghgjkgghklhjgkghghjjkjhjkkggjkhjkhjjhhfhjhklkhkhjjklzzzyyyhjggjhgghhjhNWNWBWHWPPN222kkkgtytghjjh343445499}
and
\er{MaxVacFull1ninshtrgravortghhghgjkgghklhjgkghghjjkjhjkkggjkhjkhjjhhfhjhklkhkhjjklzzzyyyhjggjhgghhjhNWNWBWHWPPN222kkkgtytghjjhghhg343445499}
in the system $\bf(*)$ we have
\begin{equation}\label{guigjgjffghguygjyfDDnw22khhjhgg22mmhkkhhjjjg99}
\begin{cases}
curl_{\vec x}\vec B=\frac{4\pi}{c}\vec
j+\frac{1}{c}\frac{\partial\vec D}{\partial t}-\frac{1}{c}curl_{\vec
x}\left(\vec v\times \vec
D\right),\\
div_{\vec x}\vec D=4\pi\rho,\\
curl_{\vec x}\vec D+\frac{1}{c}\frac{\partial\vec B}{\partial t}
-\frac{1}{c}\,curl_{\vec x}\left(\vec v\times\vec B\right)=0,\\
div_{\vec x}\vec B=0,\\
\frac{1}{c}\frac{d\vec W}{\partial t}(t)+\vec
R=-\frac{1}{c}\left(\frac{\partial\vec v}{\partial t}+d_{\vec x}\vec
v\cdot\vec v\right),\\
\vec Q=curl_{\vec x}\vec v,
\\
\frac{4\pi G}{c^2}\left(\mu g'\left(\frac{1}{2}\left|\vec u-\vec
v\right|^2\right)\left(\vec u-\vec v\right)+\frac{1}{4\pi c}\vec
D\times\vec
B-\frac{c}{4\pi G}\vec R\times\vec Q\right)=\\
(1+\beta)\,curl_{\vec x}\vec Q
-\frac{1}{c}\left(\frac{\partial\vec R}{\partial t}-curl_{\vec
x}\left\{\vec v\times\vec R\right\}+\left(div_{\vec x}\vec
R\right)\vec v\right),\\
\Div_{\vec x}\vec R=\frac{4\pi G}{c} (\mu+Q_0),\\
\frac{\partial Q_0}{\partial t}+div_{\vec x}\left(Q_0\vec
v\right)=-\Div_{\vec x}\left\{\mu
\left(g'\left(\frac{1}{2}\left|\vec u-\vec
v\right|^2\right)-1\right)(\vec u-\vec v)+\frac{1}{4\pi c}\vec
D\times\vec B-\frac{c}{4\pi  G}\vec R\times\vec Q\right\}.
\end{cases}
\end{equation}
On the other hand, since the system $\bf(**)$ is inertial, the
quantity $\frac{d\vec W}{\partial t}(t)$, being generated by the
gravitational field from the far bodies, is insignificant
with respect to the quantity $c\vec R$ in the scale compatible to
the Earth size. Thus, we rewrite
\er{guigjgjffghguygjyfDDnw22khhjhgg22mmhkkhhjjjg99} as:
\begin{equation}\label{guigjgjffghguygjyfDDnw22khhjhgg22mmhkkhhjjjgioo99}
\begin{cases}
curl_{\vec x}\vec B=\frac{4\pi}{c}\vec
j+\frac{1}{c}\frac{\partial\vec D}{\partial t}-\frac{1}{c}curl_{\vec
x}\left(\vec v\times \vec
D\right),\\
div_{\vec x}\vec D=4\pi\rho,\\
curl_{\vec x}\vec D+\frac{1}{c}\frac{\partial\vec B}{\partial t}
-\frac{1}{c}\,curl_{\vec x}\left(\vec v\times\vec B\right)=0,\\
div_{\vec x}\vec B=0,\\
\vec R=-\frac{1}{c}\left(\frac{\partial\vec v}{\partial t}+d_{\vec
x}\vec
v\cdot\vec v\right),\\
\vec Q=curl_{\vec x}\vec v,
\\
\frac{4\pi G}{c^2}\left(\mu g'\left(\frac{1}{2}\left|\vec u-\vec
v\right|^2\right)\left(\vec u-\vec v\right)+\frac{1}{4\pi c}\vec
D\times\vec
B-\frac{c}{4\pi G}\vec R\times\vec Q\right)=\\
(1+\beta)\,curl_{\vec x}\vec Q
-\frac{1}{c}\left(\frac{\partial\vec R}{\partial t}-curl_{\vec
x}\left\{\vec v\times\vec R\right\}+\left(div_{\vec x}\vec
R\right)\vec v\right),\\
\Div_{\vec x}\vec R=\frac{4\pi G}{c} (\mu+Q_0),\\
\frac{\partial Q_0}{\partial t}+div_{\vec x}\left(Q_0\vec
v\right)=-\Div_{\vec x}\left\{\mu
\left(g'\left(\frac{1}{2}\left|\vec u-\vec
v\right|^2\right)-1\right)(\vec u-\vec v)+\frac{1}{4\pi c}\vec
D\times\vec B-\frac{c}{4\pi  G}\vec R\times\vec Q\right\}.
\end{cases}
\end{equation}
Next, we can neglect all the far cosmic body masses expect of the
Earth itself and thus we can consider
\begin{equation}
\label{MaxVacFull1ninshtrgravortghhghgjkgghklhjgkghghjjkjhjkkggjkhjkhjjhhfhjhklkhkhjjklzzzyyyhjggjhgghhjhNWNWBWHWPPN222kkkjhjhjjhjhjhjhhjhj3444jkkyipi99}
\mu(\vec x,t)=\mu_1\left(|\vec x|\right)\quad\text{and}\quad\vec
u(\vec x,t)=0 \quad\quad\forall\,(\vec x,t),
\end{equation}
where $\mu_1:=\mu_1\left(|\vec x|\right)$ is the inertial mass
density of the Earth which is assumed to be a radial function such
that
\begin{equation}
\label{MaxVacFull1ninshtrgravortghhghgjkgghklhjgkghghjjkjhjkkggjkhjkhjjhhfhjhklkhkhjjklzzzyyyhjggjhgghhjhNWNWBWHWPPN222kkkjhjhjjhjhjhjhhjhj3444jkkyipioyiyi99}
\mu_1\left(|\vec x|\right)=0\quad\quad\text{if}\quad |\vec x|>r_0,
\end{equation}
where $r_0$ is the Earth radius. Moreover, we can neglect all the
electromagnetical masses and thus we simplify the equations for the
Gravity in \er{guigjgjffghguygjyfDDnw22khhjhgg22mmhkkhhjjjgioo} as:
\begin{equation}\label{guigjgjffghguygjyfDDnw22khhjhgg22mmhkkhhjjjgioojkjkggjg99}
\begin{cases}
\vec R=-\frac{1}{c}\left(\frac{\partial\vec v}{\partial t}+d_{\vec
x}\vec
v\cdot\vec v\right),\\
\vec Q=curl_{\vec x}\vec v,
\\
\frac{4\pi G}{c^2}\left(-\mu_1\left(|\vec x|\right)
g'\left(\frac{1}{2}\left|\vec
v\right|^2\right)\vec v-\frac{c}{4\pi G}\vec R\times\vec Q\right)=\\
(1+\beta)\,curl_{\vec x}\vec Q
-\frac{1}{c}\left(\frac{\partial\vec R}{\partial t}-curl_{\vec
x}\left\{\vec v\times\vec R\right\}+\left(div_{\vec x}\vec
R\right)\vec v\right),\\
\Div_{\vec x}\vec R=\frac{4\pi G}{c} \left(\mu_1\left(|\vec x|\right)+Q_0\right),\\
\frac{\partial Q_0}{\partial t}+div_{\vec x}\left(Q_0\vec
v\right)=\Div_{\vec x}\left\{\mu_1\left(|\vec x|\right)
\left(g'\left(\frac{1}{2}\left|\vec v\right|^2\right)-1\right)\vec
v+\frac{c}{4\pi  G}\vec R\times\vec Q\right\}.
\end{cases}
\end{equation}
Being in the system  $\bf(*)$ which is stationary with respect to
the center of the Earth we look for stationary (i.e. time
independent) solutions of
\er{guigjgjffghguygjyfDDnw22khhjhgg22mmhkkhhjjjgioojkjkggjg99}. Thus
\er{guigjgjffghguygjyfDDnw22khhjhgg22mmhkkhhjjjgioojkjkggjg99}
implies:
\begin{equation}\label{guigjgjffghguygjyfDDnw22khhjhgg22mmhkkhhjjjgioojkjkggjghjggjhh99}
\begin{cases}
\vec R=-\frac{1}{c}\,d_{\vec x}\vec
v\cdot\vec v,\\
\vec Q=curl_{\vec x}\vec v,
\\
\frac{4\pi G}{c^2}\left(-\mu_1\left(|\vec
x|\right)g'\left(\frac{1}{2}\left|\vec
v\right|^2\right)\vec v-\frac{c}{4\pi G}\vec R\times\vec Q\right)=\\
(1+\beta)\,curl_{\vec x}\vec Q
-\frac{1}{c}\left(-curl_{\vec x}\left\{\vec v\times\vec
R\right\}+\left(div_{\vec x}\vec
R\right)\vec v\right),\\
\Div_{\vec x}\vec R=\frac{4\pi G}{c} \left(\mu_1\left(|\vec x|\right)+Q_0\right),\\
div_{\vec x}\left(Q_0\vec v\right)=\Div_{\vec
x}\left\{\mu_1\left(|\vec x|\right)
\left(g'\left(\frac{1}{2}\left|\vec v\right|^2\right)-1\right)\vec
v+\frac{c}{4\pi G}\vec R\times\vec Q\right\}.
\end{cases}
\end{equation}
On the other hand, by the symmetry considerations of the problem we
look for the solution of
\er{guigjgjffghguygjyfDDnw22khhjhgg22mmhkkhhjjjgioojkjkggjghjggjhh99}
that satisfies $\vec v(\vec x)=\nabla_{\vec x}Z_0\left(|\vec
x|\right)$ where, again by the symmetry of the problem, the scalar
function $Z_0\left(|\vec x|\right)$ should be radial. In particular,
by
\er{guigjgjffghguygjyfDDnw22khhjhgg22mmhkkhhjjjgioojkjkggjghjggjhh99}
we obtain
\begin{equation}\label{guigjgjffghguygjyfDDnjggjgjgj99}
\vec Q=curl_{\vec x}\vec v=0
\end{equation}
and thus we simplify
\er{guigjgjffghguygjyfDDnw22khhjhgg22mmhkkhhjjjgioojkjkggjghjggjhh99}
as:
\begin{equation}\label{guigjgjffghguygjyfDDnw22khhjhgg22mmhkkhhjjjgioojkjkggjghjggjhhuiiyyiyijkjkhkkh99}
\begin{cases}
\vec R=-\frac{1}{c}\,d_{\vec x}\vec v\cdot\vec
v=-\frac{1}{c}\,\nabla_{\vec x}\left(\frac{1}{2}\left|\vec
v\right|^2\right),\\
\vec Q=curl_{\vec x}\vec v=0,
\\
\frac{4\pi G}{c^2}\,\mu_1\left(|\vec
x|\right)g'\left(\frac{1}{2}\left|\vec v\right|^2\right)\vec v=
\frac{1}{c}\left(-curl_{\vec x}\left\{\vec v\times\vec
R\right\}+\left(div_{\vec x}\vec
R\right)\vec v\right),\\
\Div_{\vec x}\vec R=\frac{4\pi G}{c}\,\mu_1\left(|\vec
x|\right)\,g'\left(\frac{1}{2}\left|\vec
v\right|^2\right),\\
curl_{\vec x}\vec R=0,\\
Q_0=\mu_1\left(|\vec x|\right) \left(g'\left(\frac{1}{2}\left|\vec
v\right|^2\right)-1\right).
\end{cases}
\end{equation}
In particular, since $\vec v=\nabla_{\vec x}Z_0\left(|\vec
x|\right)$ and $\vec R=-\frac{1}{c}\,\nabla_{\vec
x}\left(\frac{1}{2}\left|\vec v\right|^2\right)$ are both gradients
of radial functions, we have $\vec v\times\vec R=0$ and thus, we
further simplify
\er{guigjgjffghguygjyfDDnw22khhjhgg22mmhkkhhjjjgioojkjkggjghjggjhhuiiyyiyijkjkhkkh99}
as:
\begin{equation}\label{guigjgjffghguygjyfDDnw22khhjhgg22mmhkkhhjjjgioojkjkggjghjggjhhuiiyyiyijkjkhkkhuu99}
\begin{cases}
\vec R=-\frac{1}{c}\,\nabla_{\vec x}\left(\frac{1}{2}\left|\vec
v\right|^2\right),\\
\vec Q=curl_{\vec x}\vec v=0,
\\
\Div_{\vec x}\vec R=\frac{4\pi G}{c}\,\mu_1\left(|\vec
x|\right)\,g'\left(\frac{1}{2}\left|\vec
v\right|^2\right),\\
curl_{\vec x}\vec R=0,\\
Q_0=\mu_1\left(|\vec x|\right) \left(g'\left(\frac{1}{2}\left|\vec
v\right|^2\right)-1\right).
\end{cases}
\end{equation}
However,
\er{guigjgjffghguygjyfDDnw22khhjhgg22mmhkkhhjjjgioojkjkggjghjggjhhuiiyyiyijkjkhkkhuu99}
is equivalent to the following:
\begin{equation}\label{guigjgjffghguygjyfDDnw22khhjhgg22mmhkkhhjjjgioojkjkggjghjggjhhuiiyyiyijkjkhkkhuuhhhhhjgjikuo99}
\begin{cases}
\Delta_{\vec x}\left(-\frac{1}{2}\left|\vec v\right|^2\right)=4\pi
G\,\mu_1\left(|\vec x|\right)\,g'\left(\frac{1}{2}\left|\vec
v\right|^2\right),\\
curl_{\vec x}\vec v=0,\\
\vec R=-\frac{1}{c}\,\nabla_{\vec x}\left(\frac{1}{2}\left|\vec
v\right|^2\right),\\
\vec Q=0,
\\
Q_0=\mu_1\left(|\vec x|\right) \left(g'\left(\frac{1}{2}\left|\vec
v\right|^2\right)-1\right).
\end{cases}
\end{equation}
Therefore, denoting
\begin{equation}\label{gjfghdhgdfjfjfgukiufirfiufuy99}
\Phi_1:=-\frac{1}{2}\left|\vec v\right|^2
\end{equation}
we rewrite
\er{guigjgjffghguygjyfDDnw22khhjhgg22mmhkkhhjjjgioojkjkggjghjggjhhuiiyyiyijkjkhkkhuuhhhhhjgjikuo99}
as:
\begin{equation}\label{guigjgjffghguygjyfDDnw22khhjhgg22mmhkkhhjjjgioojkjkggjghjggjhhuiiyyiyijkjkhkkhuuhhhhhjgjikuojhjhjhjguyttty99}
\begin{cases}
\Delta_{\vec x}\Phi_1=4\pi G\,\mu_1\left(|\vec
x|\right)\,g'\left(-\Phi_1\right),\\
\Phi_1:=-\frac{1}{2}\left|\vec v\right|^2,
\\
curl_{\vec x}\vec v=0,\\
\vec R=-\frac{1}{c}\,\nabla_{\vec x}\left(\frac{1}{2}\left|\vec
v\right|^2\right),\\
\vec Q=0,
\\
Q_0=\mu_1\left(|\vec x|\right)
\left(g'\left(-\Phi_1\right)-1\right),
\end{cases}
\end{equation}
where the scalar field $\Phi_1=\Phi\left(|\vec x|\right)$ is radial,
and outside of the Earth surface it coincides with the following
Newtonian potential of the Earth:
 $\Phi_1(\vec
x)=-\frac{GM_0}{|\vec x|}$, where $M_0$ is the total effective
\underline{gravitational} mass of the Earth, defined as
\begin{equation}
\label{MaxVacFull1ninshtrgravortghhghgjkgghklhjgkghghjjkjhjkkggjkhjkhjjhhfhjhklkhkhjjklzzzyyyhjggjhgghhjhNWNWBWHWPPN222kkkgtytghjjhpoppkkkhhhk3499jhikjhjhj99jkkh}
M_0=\iiint_{|\vec x|\leq r_0}\mu_1\left(|\vec
x|\right)g'\left(-\Phi_1\left(|\vec x|\right)\right)\,d\vec
x=\iiint_{|\vec x|\leq r_0}\mu_1\left(|\vec
x|\right)g'\left(\frac{1}{2}\left|\nabla_{\vec x}Z_0\left(|\vec
x|\right)\right|^2\right)\,d\vec x,
\end{equation}
(for the inertial mass of the Earth $m_0$ we have $m_0=\iiint_{|\vec
x|\leq r_0}\mu_1\left(|\vec x|\right)d\vec x$ and thus, in the
relativistic-like case, where
$g'(s)=\left(1-\frac{2s}{c^2}\right)^{-\frac{1}{2}}>1$ we have
$M_0>m_0$). Thus, since there exists a scalar radial field
$Z_0\left(|\vec x|\right)$ such that $\vec v(\vec x)=\nabla_{\vec
x}Z_0\left(|\vec x|\right)$
by
\er{guigjgjffghguygjyfDDnw22khhjhgg22mmhkkhhjjjgioojkjkggjghjggjhhuiiyyiyijkjkhkkhuuhhhhhjgjikuojhjhjhjguyttty99}
we obtain
\begin{equation}
\label{MaxVacFull1ninshtrgravortghhghgjkgghklhjgkghghjjkjhjkkggjkhjkhjjhhfhjhklkhkhjjklzzzyyyhjggjhgghhjhNWNWBWHWPPN222kkkgtytghjjhpoppkkkhhhk3499}
\left|\frac{dZ_0}{d(|\vec x|)}(|\vec x|)\right|= \sqrt{-2\Phi_1(\vec
x)},
\end{equation}
that implies either
\begin{equation}
\label{MaxVacFull1ninshtrgravortghhghgjkgghklhjgkghghjjkjhjkkggjkhjkhjjhhfhjhklkhkhjjklzzzyyyhjggjhgghhjhNWNWBWHWPPN222kkkgtytghjjhpoppkkkhhhkhjhj3499}
\vec v(\vec x)= \frac{\sqrt{-2\Phi_1(|\vec x|)}}{|\vec x|}\vec x,
\end{equation}
or
\begin{equation}
\label{MaxVacFull1ninshtrgravortghhghgjkgghklhjgkghghjjkjhjkkggjkhjkhjjhhfhjhklkhkhjjklzzzyyyhjggjhgghhjhNWNWBWHWPPN222kkkgtytghjjhpoppkkkhhhkhjhjiuu3499}
\vec v(\vec x)= -\frac{\sqrt{-2\Phi_1(|\vec x|)}}{|\vec x|}\vec x,
\end{equation}
exactly as in the case of the usual Newtonian gravity. In
particular, on the Earth surface we have:
\begin{equation}
\label{MaxVacFull1ninshtrgravortghhghgjkgghklhjgkghghjjkjhjkkggjkhjkhjjhhfhjhklkhkhjjklzzzyyyhjggjhgghhjhNWNWBWHWPPN222kkkgtytghjjhpoppkkkhhhkhjhjiuuokokok3499}
|\vec v|=\sqrt{\frac{2GM_0}{r_0}},
\end{equation}
where $r_0$ is the Earth radius and $M_0$ is the total effective
gravitational mass of the Earth, i.e. the absolute value of the
vectorial gravitational potential on the Earth surface approximately
equals to the escape velocity and its direction is normal to the
Earth, either downward or upward.

Finally, note that the same solution as in
\er{guigjgjffghguygjyfDDnw22khhjhgg22mmhkkhhjjjgioojkjkggjghjggjhhuiiyyiyijkjkhkkhuuhhhhhjgjikuojhjhjhjguyttty99}
and
\er{MaxVacFull1ninshtrgravortghhghgjkgghklhjgkghghjjkjhjkkggjkhjkhjjhhfhjhklkhkhjjklzzzyyyhjggjhgghhjhNWNWBWHWPPN222kkkgtytghjjhpoppkkkhhhkhjhj3499}
or
\er{MaxVacFull1ninshtrgravortghhghgjkgghklhjgkghghjjkjhjkkggjkhjkhjjhhfhjhklkhkhjjklzzzyyyhjggjhgghhjhNWNWBWHWPPN222kkkgtytghjjhpoppkkkhhhkhjhjiuu3499}
is valid also for the laws of the Newtonian-type gravity given by
\er{MaxVacFull1bjkgjhjhgjaaahkjhhjzzzyyykkknnnNWBWHWPPNNewCCmmhhjjhkkk}.
We can get this either by a direct calculations or by passing to the
limit $\beta\to+\infty$ in \er{guigjgjffghguygjyfDDnw22mm34}.
\end{remark}

\subsection{Covariant formulation of the laws of gravity in cartesian and
curvilinear coordinate systems}\label{dvfgdgdkjg}
\subsubsection{The case of the Newtonian type gravity}\label{gggggyffu997tg7gffttft}
In this subsection we find a equivalent form of the Lagrangian
density of the gravitational-electromagnetic field in the case of
the Newtonian-type gravity having the general form
\er{vhfffngghkjgghDDmm}:
\begin{multline}\label{vhfffngghkjgghDDmmioiuiu}
L\left(\vec A,\Psi,\vec v,\Phi,\vec p,\vec
x,t\right):=\frac{1}{8\pi}\left|-\nabla_{\vec
x}\Psi-\frac{1}{c}\frac{\partial\vec A}{\partial t}+\frac{1}{c}\vec
v\times curl_{\vec x}\vec A\right|^2-\frac{1}{8\pi}\left|curl_{\vec
x}\vec A\right|^2-\left(\rho\Psi-\frac{1}{c}\vec A\cdot\vec
j\right)\\+\mu g\left(\frac{1}{2}\left|\vec u-\vec
v\right|^2\right)+
\frac{1}{2}\left(d_{\vec x}\vec v+\left\{d_{\vec x}\vec
v\right\}^T\right):\left(d_{\vec x}\vec p+\left\{d_{\vec x}\vec
p\right\}^T\right)-2\left(div_{\vec x}\vec v\right)\left(div_{\vec
x}\vec p\right)\\+\frac{1}{4\pi G}\left(div_{\vec x}\vec
v\right)\left(\frac{\partial \Phi}{\partial t}+\vec
v\cdot\nabla_{\vec x} \Phi\right)+\frac{1}{4\pi
G}\Phi\left(div_{\vec x}\vec v\right)^2-\frac{\Phi}{16\pi
G}\left|d_{\vec x}\vec v+\left\{d_{\vec x}\vec
v\right\}^T\right|^2+\frac{1}{8\pi G}\left|\nabla_{\vec
x}\Phi\right|^2.
\end{multline}
Our purpose is to make the equivalent form of this Lagrangian
density to be covariant and valid in every \underline{curvilinear}
coordinate system.

Assume first, that we deal with a cartesian inertial or non-inertial
coordinate system. Then consider a three-dimensional vectorial
gravitational potential $\vec v=(v^1,v^2,v^3)$ and consider the
covariant pseudometric tensor $G=\{g_{ij}\}_{0\leq i,j\leq 3}$
defined by \er{hoyuiouigyfg3478zzrrZZffGGhhjhj}:
\begin{equation}\label{hoyuiouigyfg3478zzrrZZff}
\begin{cases}
g_{00}=1-\frac{|\vec v|^2}{c^2}\\
g_{ij}=-\delta_{ij}\quad\forall 1\leq i,j\leq 3\\
g_{0j}=g_{j0}=\frac{v^j}{c}\quad\forall 1\leq j\leq 3.
\end{cases}
\end{equation}
Next consider the contravariant pseudometric tensor $\tilde
G=\{g^{ij}\}_{0\leq i,j\leq 3}$ defined by
\er{hoyuiouigyfghgjh3478zzrrZZffGGjkkj}:
\begin{equation}\label{hoyuiouigyfghgjh3478zzrrZZff}
\begin{cases}
g^{00}=1\\
g^{ij}=-\delta_{ij}+\frac{v^iv^j}{c^2}\quad\forall 1\leq i,j\leq 3\\
g^{0j}=g^{j0}=\frac{v^j}{c}\quad\forall 1\leq j\leq 3.
\end{cases}
\end{equation}

Next consider the Christoffel Symbols:
\begin{equation}\label{jhffhgffgh3478zzrrZZff}
\begin{cases}
\Gamma_{i,kn}:=\frac{1}{2}\left(\frac{\partial g_{ik}}{\partial
x_{n}}+\frac{\partial g_{in}}{\partial x_{k}}-\frac{\partial
g_{kn}}{\partial x_{i}}\right)\\
\Gamma^{i}_{kn}:=\sum_{j=0}^{3}g^{ij}\Gamma_{j,kn}
\end{cases}\quad\quad\forall\, i,k,n=0,1,2,3,
\end{equation}
where $\vec x=(x_1,x_2,x_3)$, $x_0=ct$ and the point in the four
dimensional space-time is denoted as $(x^0,x^1,x^2,x^3):=(ct,\vec
x)=(x_0,x_1,x_2,x_3)$. In particular by
\er{hoyuiouigyfg3478zzrrZZff} and by the first equation in
\er{jhffhgffgh3478zzrrZZff} we obtain:
\begin{equation}\label{jhffhgffghhjgj3478zzrrZZff}
\begin{cases}
\Gamma_{0,00}=-\frac{1}{2c^3}\,\frac{\partial (|\vec v|^2)}{\partial t}\\
\Gamma_{0,k0}=\Gamma_{0,0k}=-\frac{1}{2c^2}\frac{\partial (|\vec
v|^2)}{\partial
x_{k}}\quad\quad\forall\, k=1,2,3,\\
\Gamma_{0,kn}=\frac{1}{2c}\left(\frac{\partial v^{k}}{\partial
x_{n}}+\frac{\partial v^{n}}{\partial
x_{k}}\right)\quad\quad\forall\, k,n=1,2,3
\\
\Gamma_{i,00}=\frac{1}{c^2}\,\left(\frac{\partial v^{i}}{\partial
t}+\frac{1}{2}\frac{\partial (|\vec v|^2)}{\partial
x_{i}}\right)\quad\quad\forall\, i=1,2,3,
\\
\Gamma_{i,k0}=\Gamma_{i,0k}=\frac{1}{2c}\left(\frac{\partial
v^{i}}{\partial x_{k}}-\frac{\partial v^{k}}{\partial
x_{i}}\right)\quad\quad\forall\, i,k=1,2,3,
\\
\Gamma_{i,kn}=0\quad\quad\forall\, i,k,n=1,2,3.
\end{cases}
\end{equation}

Next given a four-covector field $(S_0,S_1,S_2,S_3)$ on the group
$\mathcal{S}_0$ and the corresponding lifted four-vector
$(S^0,S^1,S^2,S^3)$ given by
\begin{equation}\label{fgjfjhgghhgjghjhjkkkkgjghghuiiiulkkjKKjkhkhh}
(S^0,S^1,S^2,S^3):=\Big\{\sum_{k=0}^{3}g^{mk}S_{k}\Big\}_{m=0,1,2,3},
\end{equation}
by \er{hoyuiouigyfghgjh3478zzrrZZffGGjkkj},
\er{hoyuiouigyfg3478zzrrZZffGGhhjhj} and
\er{fgjfjhgghhgjghjhjkkkkgjghghuiiiulkkjKK} we have:
\begin{equation}\label{fgjfjhgghhgjghjhjkkkkgjghghuiiiulkkjKKyuyyujojkjk}
S^0=S_{0}+\sum_{k=1}^{3}\frac{1}{c}v^k S_{k}\quad\text{and}\quad
S^m=-S_{m}+\frac{1}{c}\bigg(S_{0}+\sum_{k=1}^{3}\frac{1}{c}v^k
S_{k}\bigg)v^m\quad\quad\forall m=1,2,3,
\end{equation}
On the other hand, if we denote:
\begin{equation}\label{hoyuiouigyfghgjh3478zzrrZZjhjhhjkjjhnmnmffjjpopo}
\begin{cases}
\phi:=S^0\\
h_j:=-S_j\quad\quad\forall 1\leq j\leq 3\quad\text{and}\quad\vec
h:=(h_1,h_2,h_3),
\end{cases}
\end{equation}
then, as we get above in subsection \ref{jjjjjkkk}, $\phi$ is a
proper scalar field and $\vec h$ is a proper vector field and by
\er{fgjfjhgghhgjghjhjkkkkgjghghuiiiulkkjKKyuyyujojkjk} we can write:
\begin{equation}\label{hoyuiouigyfghgjh3478zzrrZZjhjhhjkjjhnmnmffjj}
\begin{cases}
S_0=\phi+\sum_{k=1}^{3}\frac{1}{c}v^k h_{k}=\phi+\frac{1}{c}\vec v\cdot\vec h\\
S_j=-h_j\quad\quad\forall 1\leq j\leq 3.
\end{cases}
\end{equation}
and
\begin{equation}\label{hoyuiouigyfghgjh3478zzrrZZjhjhhjkjjhnmnmffjjkkk}
\begin{cases}
S^0=\phi\\
S^j=\phi \frac{v^j}{c}+h_j\quad\quad\forall 1\leq j\leq 3.
\end{cases}
\end{equation}
Next consider $Z_{ij}=\delta_j S_i$, where $\delta_j$ means the
covariant derivative with respect to the dynamical pseudo-metrics
$G=\{g_{ij}\}_{0\leq i,j\leq 3}$. Then
\begin{equation}\label{hoyuiouigyfghgjh3478zzrrZZjhjhugyuyughggjhjhjyuuygtyffjj}
Z_{ij}:=\delta_j S_i:=\frac{\partial S_i}{\partial
x_j}-\sum_{k=0}^{3}\Gamma^{k}_{ij}S_k\quad\quad\forall\, 0\leq
i,j\leq 3.
\end{equation}
It is well known from the Tensor Analysis that, given a co-vector
field $S_j$ on the group $\mathcal{S}_0$, $Z_{ij}$, defined by
\er{hoyuiouigyfghgjh3478zzrrZZjhjhugyuyughggjhjhjyuuygtyffjj}, is a
two times covariant tensor field on the group $\mathcal{S}_0$. On
the other hand we can write by
\er{hoyuiouigyfghgjh3478zzrrZZjhjhugyuyughggjhjhjyuuygtyffjj} and by
the second equality in \er{jhffhgffgh3478zzrrZZff}:
\begin{equation}\label{hoyuiouigyfghgjh3478zzrrZZjhjhugyuyughggjhjhjyuuygtyffjjjkj}
Z_{ij}:=\delta_j S_i:=\frac{\partial S_i}{\partial
x_j}-\sum_{k=0}^{3}\sum_{m=0}^{3}g^{km}\Gamma_{m,ij}\,
S_k=\frac{\partial S_i}{\partial x_j}-\sum_{m=0}^{3}\Gamma_{m,ij}\,
S^m\quad\quad\forall\, 0\leq i,j\leq 3,
\end{equation}
Thus by inserting \er{hoyuiouigyfghgjh3478zzrrZZjhjhhjkjjhnmnmffjj}
and \er{hoyuiouigyfghgjh3478zzrrZZjhjhhjkjjhnmnmffjjkkk} into
\er{hoyuiouigyfghgjh3478zzrrZZjhjhugyuyughggjhjhjyuuygtyffjjjkj} we
deduce:
\begin{multline}\label{hoyuiouigyfghgjh3478zzrrZZjhjhugyuyughggjhjhjyuuygtyffjjjkjljkjj}
Z_{ij}:=\delta_j S_i=\frac{\partial S_i}{\partial
x_j}-\Gamma_{0,ij}\, S^0-\sum_{m=1}^{3}\Gamma_{m,ij}\, S^m
\\=\frac{\partial S_i}{\partial x_j}-\Gamma_{0,ij}\,
\phi-\sum_{m=1}^{3}\Gamma_{m,ij}\, h_m-\sum_{m=1}^{3}\Gamma_{m,ij}\,
 \frac{\phi \,v^m}{c}\quad\quad\forall\, 0\leq i,j\leq 3,
\end{multline}
In particular by inserting
\er{hoyuiouigyfghgjh3478zzrrZZjhjhhjkjjhnmnmffjj} and
\er{jhffhgffghhjgj3478zzrrZZff} into
\er{hoyuiouigyfghgjh3478zzrrZZjhjhugyuyughggjhjhjyuuygtyffjjjkjljkjj}
we deduce:
\begin{multline}\label{hoyuiouigyfghgjh3478zzrrZZjhjhugyuyughggjhjhjyuuygtyffjjjkjljkjjljokjl}
Z_{00}:=\delta_0 S_0=\frac{\partial S_0}{\partial
x_0}-\Gamma_{0,00}\, \phi-\sum_{m=1}^{3}\Gamma_{m,00}\,
h_m-\sum_{m=1}^{3}\Gamma_{m,00}\, \frac{\phi \,v^m}{c}=
\\
\frac{1}{c}\frac{\partial}{\partial t}\left(\phi+\frac{1}{c}\vec
v\cdot\vec h\right)+\frac{\phi}{2c^3}\,\frac{\partial (|\vec
v|^2)}{\partial t}-\sum_{m=1}^{3}\frac{1}{c^2}\,\left(\frac{\partial
v^{m}}{\partial t}+\frac{1}{2}\frac{\partial (|\vec v|^2)}{\partial
x_{m}}\right)\,
h_m-\sum_{m=1}^{3}\frac{1}{c^2}\,\left(\frac{\partial
v^{m}}{\partial t}+\frac{1}{2}\frac{\partial (|\vec v|^2)}{\partial
x_{m}}\right)\,
 \frac{\phi \,v^m}{c}\\=\frac{1}{c}\frac{\partial\phi}{\partial t}+\frac{1}{c^2}\vec v \cdot\frac{\partial\vec h}{\partial
 t}-\sum_{m=1}^{3}\frac{1}{2c^2}\,\left(\frac{\partial (|\vec v|^2)}{\partial
x_{m}}\right)\, \left(h_m+\frac{\phi \,v^m}{c}\right)
\end{multline}
\begin{multline}\label{hoyuiouigyfghgjh3478zzrrZZjhjhugyuyughggjhjhjyuuygtyffjjjkjljkjjljokjlpiikk}
Z_{0j}:=\delta_j S_0=\frac{\partial S_0}{\partial
x_j}-\Gamma_{0,0j}\, \phi-\sum_{m=1}^{3}\Gamma_{m,0j}\,
h_m-\sum_{m=1}^{3}\Gamma_{m,0j}\,
 \frac{\phi \,v^m}{c}=\\ \frac{\partial}{\partial
x_j}\left(\phi+\frac{1}{c}\vec v\cdot\vec
h\right)+\frac{\phi}{2c^2}\frac{\partial (|\vec v|^2)}{\partial
x_{j}}-\sum_{m=1}^{3}\frac{1}{2c}\left(\frac{\partial
v^{m}}{\partial x^{j}}-\frac{\partial v^{j}}{\partial
x_{m}}\right)\, h_m-\sum_{m=1}^{3}\frac{1}{2c}\left(\frac{\partial
v^{m}}{\partial x_{j}}-\frac{\partial v^{j}}{\partial
x_{m}}\right)\,
 \frac{\phi \,v^m}{c}
 =\\ \frac{\partial\phi}{\partial
x_j}+\sum_{m=1}^{3}\frac{v^m}{c}\frac{\partial h_m}{\partial x_j}
+\sum_{m=1}^{3}\frac{1}{2c}\left(\frac{\partial v^{j}}{\partial
x_{m}}+\frac{\partial v^{m}}{\partial x_{j}}\right)\,
h_m+\sum_{m=1}^{3}\frac{1}{2c}\left(\frac{\partial v^{j}}{\partial
x_{m}}+\frac{\partial v^{m}}{\partial x_{j}}\right)\,
 \frac{\phi \,v^m}{c}\quad\quad\forall\, 1\leq j\leq 3,
\end{multline}
\begin{multline}\label{hoyuiouigyfghgjh3478zzrrZZjhjhugyuyughggjhjhjyuuygtyffjjjkjljkjjljokjlkllkk}
Z_{i0}:=\delta_0 S_i=\frac{\partial S_i}{\partial
x_0}-\Gamma_{0,i0}\, \phi-\sum_{m=1}^{3}\Gamma_{m,i0}\,
h_m-\sum_{m=1}^{3}\Gamma_{m,i0}\,
 \frac{\phi \,v^m}{c}=\\ -\frac{1}{c}\frac{\partial h_i}{\partial
t}+\frac{\phi}{2c^2}\frac{\partial (|\vec v|^2)}{\partial
x_{i}}-\sum_{m=1}^{3}\frac{1}{2c}\left(\frac{\partial
v^{m}}{\partial x_{i}}-\frac{\partial v^{i}}{\partial
x_{m}}\right)\, h_m-\sum_{m=1}^{3}\frac{1}{2c}\left(\frac{\partial
v^{m}}{\partial x_{i}}-\frac{\partial v^{i}}{\partial
x_{m}}\right)\,
 \frac{\phi \,v^m}{c}
=\\ -\frac{1}{c}\frac{\partial h_i}{\partial
t}+\sum_{m=1}^{3}\frac{1}{2c}\left(\frac{\partial v^{i}}{\partial
x_{m}}-\frac{\partial v^{m}}{\partial x_{i}}\right)\,
h_m+\sum_{m=1}^{3}\frac{1}{2c}\left(\frac{\partial v^{m}}{\partial
x_{i}}+\frac{\partial v^{i}}{\partial x_{m}}\right)\,
 \frac{\phi \,v^m}{c}
 \quad\quad\forall\, 1\leq i\leq 3,
\end{multline}
and
\begin{multline}\label{hoyuiouigyfghgjh3478zzrrZZjhjhugyuyughggjhjhjyuuygtyffjjjkjljkjjljokjllll}
Z_{ij}:=\delta_j S_i=\frac{\partial S_i}{\partial
x_j}-\Gamma_{0,ij}\, \phi-\sum_{m=1}^{3}\Gamma_{m,ij}\,
h_m-\sum_{m=1}^{3}\Gamma_{m,ij}\,
 \frac{\phi \,v^m}{c}\\=
-\left(\frac{\partial h_i}{\partial
x_j}+\frac{\phi}{2c}\left(\frac{\partial v^{i}}{\partial
x_{j}}+\frac{\partial v^{j}}{\partial x_{i}}\right)\right)
 \quad\quad\forall\, 1\leq i,j\leq 3.
\end{multline}
Therefore, if we denote the symmetric two times covariant tensor
$\hat Z_{ij}$ defined by
\begin{equation}\label{hoyuiouigyfghgjh3478zzrrZZjhjhugyuyughggjhjhjyuuygtyffjjkpkkk}
\hat Z_{ij}:=Z_{ij}+Z_{ji}=\delta_j S_i+\delta_i S_j=\frac{\partial
S_i}{\partial x_j}+\frac{\partial S_j}{\partial
x_i}-\sum_{k=0}^{3}2\Gamma^{k}_{ij}S_k\quad\quad\forall\, 0\leq
i,j\leq 3,
\end{equation}
then by
\er{hoyuiouigyfghgjh3478zzrrZZjhjhugyuyughggjhjhjyuuygtyffjjjkjljkjjljokjl},
\er{hoyuiouigyfghgjh3478zzrrZZjhjhugyuyughggjhjhjyuuygtyffjjjkjljkjjljokjlpiikk},
\er{hoyuiouigyfghgjh3478zzrrZZjhjhugyuyughggjhjhjyuuygtyffjjjkjljkjjljokjlkllkk}
and
\er{hoyuiouigyfghgjh3478zzrrZZjhjhugyuyughggjhjhjyuuygtyffjjjkjljkjjljokjllll}
we have:
\begin{multline}\label{hoyuiouigyfghgjh3478zzrrZZjhjhugyuyughggjhjhjyuuygtyffjjjkjljkjjljokjljjkjkopopkklkl}
\hat Z_{00}=\frac{2}{c}\frac{\partial\phi}{\partial
t}+\sum_{m=1}^{3}\frac{2v^m}{c^2}\frac{\partial h_m}{\partial
 t}-\sum_{m=1}^{3}\frac{1}{c^2}\,\left(\frac{\partial (|\vec v|^2)}{\partial
x_{m}}\right)\, \left(h_m+\frac{\phi \,v^m}{c}\right)=\\
\frac{2}{c}\left(\frac{\partial\phi}{\partial
t}+\sum_{m=1}^{3}v^m\frac{\partial \phi}{\partial
x_m}\right)-\sum_{m=1}^{3}\frac{2v^m}{c}\left(\frac{\partial
\phi}{\partial x_m}-\frac{1}{c}\frac{\partial h_m}{\partial
 t}-\sum_{k=1}^{3}\frac{v^k}{c}\frac{\partial h_{m}}{\partial
x_{k}}+\sum_{k=1}^{3}\frac{h_k}{c}\frac{\partial v^{m}}{\partial
x_{k}}\right)\\-\sum_{m=1}^{3}\frac{v^m}{c}\left(\sum_{k=1}^{3}\frac{v_k}{c}\left(\left(\frac{\partial
h_{k}}{\partial x_{m}}+\frac{\partial h_{m}}{\partial
x_{k}}\right)+\left(\frac{\partial v^{k}}{\partial
x_{m}}+\frac{\partial v^{m}}{\partial
x_{k}}\right)\frac{\phi}{c}\right)\right)
\end{multline}
\begin{multline}\label{hoyuiouigyfghgjh3478zzrrZZjhjhugyuyughggjhjhjyuuygtyffjjjkjljkjjljokjlpiikkkk}
\hat Z_{0j}=\hat Z_{j0}= \frac{\partial\phi}{\partial
x_j}+\sum_{m=1}^{3}\frac{v^m}{c}\frac{\partial h_m}{\partial x_j}
+\sum_{m=1}^{3}\frac{1}{2c}\left(\frac{\partial v^{j}}{\partial
x_{m}}+\frac{\partial v^{m}}{\partial x_{j}}\right)\,
h_m+\sum_{m=1}^{3}\frac{1}{2c}\left(\frac{\partial v^{j}}{\partial
x_{m}}+\frac{\partial v^{m}}{\partial x_{j}}\right)\,
 \frac{\phi \,v^m}{c}\\ -\frac{1}{c}\frac{\partial h_j}{\partial
t}+\sum_{m=1}^{3}\frac{1}{2c}\left(\frac{\partial v^{j}}{\partial
x_{m}}-\frac{\partial v^{m}}{\partial x_{j}}\right)\,
h_m+\sum_{m=1}^{3}\frac{1}{2c}\left(\frac{\partial v^{m}}{\partial
x_{j}}+\frac{\partial v^{j}}{\partial x_{m}}\right)\,
 \frac{\phi \,v^m}{c}=\\
\frac{\partial\phi}{\partial x_j}-\frac{1}{c}\frac{\partial
h_j}{\partial t}-\sum_{m=1}^{3}\frac{v^m}{c}\frac{\partial
h_j}{\partial x_m}
+\sum_{m=1}^{3}\frac{h_m}{c}\frac{\partial v^{j}}{\partial
x_{m}}+\sum_{m=1}^{3}\frac{v^m}{c}\left(\left(\frac{\partial
h_{j}}{\partial x_{m}}+\frac{\partial h_{m}}{\partial
x_{j}}\right)+\left(\frac{\partial v^{j}}{\partial
x_{m}}+\frac{\partial v^{m}}{\partial
x_{j}}\right)\frac{\phi}{c}\right)
 \\
 \quad\quad\forall\, 1\leq j\leq 3,
\end{multline}
and
\begin{equation}\label{hoyuiouigyfghgjh3478zzrrZZjhjhugyuyughggjhjhjyuuygtyffjjjkjljkjjljokjlllljkkjjk}
\hat Z_{ij}=\hat Z_{ji}= -\left(\left(\frac{\partial h_i}{\partial
x_j}+\frac{\partial h_j}{\partial
x_i}\right)+\frac{\phi}{c}\left(\frac{\partial v^{i}}{\partial
x_{j}}+\frac{\partial v^{j}}{\partial x_{i}}\right)\right)
 \quad\quad\forall\, 1\leq i,j\leq 3.
\end{equation}
Therefore, by
\er{hoyuiouigyfghgjh3478zzrrZZjhjhugyuyughggjhjhjyuuygtyffjjjkjljkjjljokjlllljkkjjk},
\er{hoyuiouigyfghgjh3478zzrrZZjhjhugyuyughggjhjhjyuuygtyffjjjkjljkjjljokjlpiikkkk}
and
\er{hoyuiouigyfghgjh3478zzrrZZjhjhugyuyughggjhjhjyuuygtyffjjjkjljkjjljokjljjkjkopopkklkl}
we obtain:
\begin{equation}\label{hoyuiouigyfghgjh3478zzrrZZjhjhugyuyughggjhjhjyuuygtyffjjjkjljkjjljokjlllljkkjjkp;;lk}
\begin{cases}\hat Z_{ij}=\hat Z_{ji}= -\left(\left(\frac{\partial
h_i}{\partial x_j}+\frac{\partial h_j}{\partial
x_i}\right)+\frac{\phi}{c}\left(\frac{\partial v^{i}}{\partial
x_{j}}+\frac{\partial v^{j}}{\partial x_{i}}\right)\right)
\quad\quad\forall\, 1\leq i,j\leq 3,\\
\hat Z_{0j}=\hat Z_{j0}= \left(\frac{\partial\phi}{\partial
x_j}-\frac{1}{c}\frac{\partial h_j}{\partial
t}-\sum_{m=1}^{3}\frac{v^m}{c}\frac{\partial h_j}{\partial x_m}
+\sum_{m=1}^{3}\frac{h_m}{c}\frac{\partial v^{j}}{\partial
x_{m}}\right)-\sum_{m=1}^{3}\frac{v^m}{c}\hat Z_{mj}
\quad\quad\forall\, 1\leq j\leq
3\\
\hat Z_{00}=\frac{2}{c}\left(\frac{\partial\phi}{\partial
t}+\sum_{m=1}^{3}v^m\frac{\partial \phi}{\partial
x_m}\right)-\sum_{m=1}^{3}\frac{2v^m}{c}\left(\hat
Z_{0m}+\sum_{k=1}^{3}\frac{v^k}{c}\hat Z_{mk}\right)+
\sum_{m=1}^{3}\sum_{k=1}^{3}\frac{v^m}{c}\frac{v^k}{c}\hat Z_{mk}
\end{cases}
\end{equation}
In particular, by
\er{hoyuiouigyfghgjh3478zzrrZZjhjhugyuyughggjhjhjyuuygtyffjjjkjljkjjljokjlllljkkjjkp;;lk}
and \er{hoyuiouigyfghgjh3478zzrrZZff} we deduce:
\begin{multline}\label{hoyuiouigyfghgjh3478zzrrZZjhjhugyuyughggjhjhjyuuygtyffjjjkjljkjjljokjljjkjkopopkklklkjjkjkjjkkj}
\sum_{i=0}^{3}\sum_{j=0}^{3}g^{ij}\hat Z_{ij}=g^{00}\hat
Z_{00}+\sum_{j=1}^{3}2g^{0j}\hat
Z_{0j}+\sum_{i=1}^{3}\sum_{j=1}^{3}g^{ij}\hat Z_{ij}= \hat
Z_{00}+\sum_{j=1}^{3}\frac{2v^j}{c}\hat
Z_{0j}+\sum_{i=1}^{3}\sum_{j=1}^{3}\frac{v^i}{c}\frac{v^j}{c}\hat
Z_{ij}-\sum_{i=1}^{3}\hat Z_{ii}\\
= \left(\frac{2}{c}\left(\frac{\partial\phi}{\partial
t}+\sum_{m=1}^{3}v^m\frac{\partial \phi}{\partial
x_m}\right)-\sum_{m=1}^{3}\frac{2v^m}{c}\left(\hat
Z_{0m}+\sum_{k=1}^{3}\frac{v^k}{c}\hat Z_{mk}\right)+
\sum_{m=1}^{3}\sum_{k=1}^{3}\frac{v^m}{c}\frac{v^k}{c}\hat
Z_{mk}\right)\\+\sum_{j=1}^{3}\frac{2v^j}{c}\hat
Z_{0j}+\sum_{i=1}^{3}\sum_{j=1}^{3}\frac{v^i}{c}\frac{v^j}{c}\hat
Z_{ij}+\sum_{i=1}^{3}2\left(\frac{\partial h_i}{\partial
x_i}+\frac{\phi}{c}\frac{\partial v^{i}}{\partial x_{i}}\right)=
2\left(\frac{1}{c}\left(\frac{\partial\phi}{\partial t}+div_{\vec
x}\left\{\phi\vec v\right\}\right)+div_{\vec x}\vec h\right).
\end{multline}

In particular, if we consider the four-dimensional gravitational
potential $(v^0,v^1,v^2,v^3)$ defined by
\er{fgjfjhgghhgjghjhjijhojihjhjjijhjjjjjuiijjjk} as:
\begin{equation}\label{fgjfjhgghhgjghjhjijhojihjhjjijhjjjjjuiijjjkhjhjhjklkk}
(v^0,v^1,v^2,v^3):=\left(1,\frac{1}{c}\vec v\right),
\end{equation}
and the corresponding lowered four-covector field
$(v_0,v_1,v_2,v_3)$, that we called the four-covector of
gravitational potential:
\begin{equation}\label{fgjfjhgghhgjghjhjijhojihjhjjijhjjjjjuiikkjjnj1kkkl}
(v_0,v_1,v_2,v_3):=\left(1,0,0,0\right),
\end{equation}
then denoting
\begin{equation}\label{fgjfjhgghhghjjgj}
\hat Y_{ij}:=\delta_j v_i+\delta_i v_j\quad\quad\forall\, 0\leq
i,j\leq 3,
\end{equation}
as a particular case of
\er{hoyuiouigyfghgjh3478zzrrZZjhjhugyuyughggjhjhjyuuygtyffjjjkjljkjjljokjlllljkkjjkp;;lk}
and
\er{hoyuiouigyfghgjh3478zzrrZZjhjhugyuyughggjhjhjyuuygtyffjjjkjljkjjljokjljjkjkopopkklklkjjkjkjjkkj}
with $\phi=1$ and $\vec h=0$, we deduce:
\begin{equation}\label{hoyuiouigyfghgjh3478zzrrZZjhjhugyuyughggjhjhjyuuygtyffjjjkjljkjjljokjlllljkkjjkp;;lkghhgghlllk;l;llklk}
\begin{cases}\hat Y_{ij}=\hat Y_{ji}= -\frac{1}{c}\left(\frac{\partial v^{i}}{\partial
x_{j}}+\frac{\partial v^{j}}{\partial x_{i}}\right)
\quad\quad\forall\, 1\leq i,j\leq 3,\\
\hat Y_{0j}=\hat Y_{j0}= -\sum_{m=1}^{3}\frac{v^m}{c}\hat Y_{mj}
\quad\quad\forall\, 1\leq j\leq
3\\
\hat Y_{00}=-\sum_{m=1}^{3}\frac{2v^m}{c}\left(\hat
Y_{0m}+\sum_{k=1}^{3}\frac{v^k}{c}\hat Y_{mk}\right)+
\sum_{m=1}^{3}\sum_{k=1}^{3}\frac{v^m}{c}\frac{v^k}{c}\hat Y_{mk},
\end{cases}
\end{equation}
and
\begin{equation}\label{hoyuiouigyfghgjh3478zzrrZZjhjhugyuyughggjhjhjyuuygtyffjjjkjljkjjljokjljjkjkopopkklklkjjkjkjjkkokkjoipio;l;l;ll;lkklkjjk}
\sum_{i=0}^{3}\sum_{j=0}^{3}g^{ij}\hat Y_{ij}=
\frac{2}{c}\left(div_{\vec x}\vec v\right).
\end{equation}
Moreover, denoting the two times contravariant lifted tensor:
\begin{equation}\label{hoyuiouigyfghgjh3478zzrrZZjhjhugyuyughggjhjhjyuuygtyffjjjkjljkjjljokjlllljkkjjkp;;lkkklijjkll}
\hat Y^{mn}:=\sum_{i=0}^{3}\sum_{j=0}^{3}g^{mj}g^{in}\hat
Y_{ij}\quad\quad\forall\, 0\leq m,n\leq 3,
\end{equation}
since then
\begin{equation}\label{hoyuiouigyfghgjh3478zzrrZZjhjhugyuyughggjhjhjyuuygtyffjjjkjljkjjljokjlllljkkjjkp;;lkkklijjkklklklll}
\hat Y^{mn}=g^{m0}g^{0n}\hat Y_{ij}+\sum_{i=1}^{3}g^{m0}g^{in}\hat
Y_{i0}+\sum_{j=1}^{3}g^{mj}g^{0n}\hat
Y_{0j}+\sum_{i=1}^{3}\sum_{j=1}^{3}g^{mj}g^{in}\hat
Y_{ij}\quad\quad\forall\, 0\leq m,n\leq 3,
\end{equation}
by
\er{hoyuiouigyfghgjh3478zzrrZZjhjhugyuyughggjhjhjyuuygtyffjjjkjljkjjljokjlllljkkjjkp;;lkghhgghlllk;l;llklk}
and \er{hoyuiouigyfghgjh3478zzrrZZff} we also deduce:
\begin{multline}\label{hoyuiouigyfghgjh3478zzrrZZjhjhugyuyughggjhjhjyuuygtyffjjjkjljkjjljokjlllljkkjjkp;;lkkklijjkklklkll;l;ghll}
\hat Y^{00}=g^{00}g^{00}\hat Y_{00}+\sum_{i=1}^{3}g^{00}g^{i0}\hat
Y_{i0}+\sum_{j=1}^{3}g^{0j}g^{00}\hat
Y_{0j}+\sum_{i=1}^{3}\sum_{j=1}^{3}g^{0j}g^{i0}\hat Y_{ij}=\\
\hat Y_{00}+\sum_{j=1}^{3}\frac{2v_j}{c}\hat
Y_{0j}+\sum_{i=1}^{3}\sum_{j=1}^{3}\frac{v^i}{c}\frac{v^j}{c}\hat
Y_{ij}=0,
\end{multline}
\begin{multline}\label{hoyuiouigyfghgjh3478zzrrZZjhjhugyuyughggjhjhjyuuygtyffjjjkjljkjjljokjlllljkkjjkp;;lkkklijjkklklkl;ll;hll}
\hat Y^{0n}=\hat Y^{n0}=g^{00}g^{0n}\hat
Y_{00}+\sum_{i=1}^{3}g^{00}g^{in}\hat
Y_{i0}+\sum_{j=1}^{3}g^{0j}g^{0n}\hat
Z_{0j}+\sum_{i=1}^{3}\sum_{j=1}^{3}g^{0j}g^{in}\hat Y_{ij}=
\frac{v^n}{c}\hat
Y_{00}\\+\sum_{i=1}^{3}\left(-\delta_{in}+\frac{v^iv^n}{c^2}\right)\hat
Y_{i0}+\sum_{j=1}^{3}\frac{v^j}{c}\frac{v^n}{c}\hat
Y_{0j}+\sum_{i=1}^{3}\sum_{j=1}^{3}\frac{v^j}{c}\left(-\delta_{in}+\frac{v^iv^n}{c^2}\right)\hat
Y_{ij} \\=-\left(\hat Y_{n0}+\sum_{j=1}^{3}\frac{v^j}{c}\hat
Y_{nj}\right)=0\quad\quad\forall\, 1\leq n\leq 3,
\end{multline}
and
\begin{multline}\label{hoyuiouigyfghgjh3478zzrrZZjhjhugyuyughggjhjhjyuuygtyffjjjkjljkjjljokjlllljkkjjkp;;lkkklijjkklklkllkjll}
\hat Y^{mn}
=g^{m0}g^{0n}\hat Y_{00}+\sum_{i=1}^{3}g^{m0}g^{in}\hat
Y_{i0}+\sum_{j=1}^{3}g^{mj}g^{0n}\hat
Y_{0j}+\sum_{i=1}^{3}\sum_{j=1}^{3}g^{mj}g^{in}\hat Y_{ij} =
\frac{v^mv^n}{c^2}\hat
Y_{00}\\+\sum_{i=1}^{3}\frac{v^m}{c}\left(-\delta_{in}+\frac{v^iv^n}{c^2}\right)\hat
Y_{i0}+\sum_{j=1}^{3}\left(-\delta_{mj}+\frac{v^mv^j}{c^2}\right)\frac{v^n}{c}\hat
Y_{0j}+\sum_{i=1}^{3}\sum_{j=1}^{3}\left(-\delta_{mj}+\frac{v^mv^j}{c^2}\right)\left(-\delta_{in}+\frac{v^iv^n}{c^2}\right)\hat
Y_{ij}
\\
= -\frac{v^m}{c}\left(\hat Y_{n0}+\sum_{j=1}^{3}\frac{v^j}{c}\hat
Y_{nj}\right)-\frac{v^n}{c}\left(\hat Y_{0m}
+\sum_{i=1}^{3}\frac{v^i}{c}\hat Y_{im}\right) +\hat Y_{mn}=
-\frac{1}{c}\left(\frac{\partial v^{m}}{\partial
x_{n}}+\frac{\partial v^{n}}{\partial x_{m}}\right) \;\;\forall
m,n=1,2,3.
\end{multline}
As a consequence of
\er{hoyuiouigyfghgjh3478zzrrZZjhjhugyuyughggjhjhjyuuygtyffjjjkjljkjjljokjlllljkkjjkp;;lkkklijjkklklkll;l;ghll},
\er{hoyuiouigyfghgjh3478zzrrZZjhjhugyuyughggjhjhjyuuygtyffjjjkjljkjjljokjlllljkkjjkp;;lkkklijjkklklkl;ll;hll}
and
\er{hoyuiouigyfghgjh3478zzrrZZjhjhugyuyughggjhjhjyuuygtyffjjjkjljkjjljokjlllljkkjjkp;;lkkklijjkklklkllkjll}
we have:
\begin{equation}\label{hoyuiouigyfghgjh3478zzrrZZjhjhugyuyughggjhjhjyuuygtyffjjjkjljkjjljokjlllljkkjjkp;;lkkljljklkpkopkpolol;l;;lklklklll}
\begin{cases}
\hat Y^{00}=0,\\
\hat Y^{0n}=\hat Y^{n0}=0\quad\quad\forall\, 1\leq n\leq 3,
\\
\hat Y^{mn}
= -\frac{1}{c}\left(\frac{\partial v^{m}}{\partial
x_{n}}+\frac{\partial v^{n}}{\partial x_{m}}\right)
\quad\quad\forall\, 1\leq m,n\leq 3.
\end{cases}
\end{equation}
Therefore, in particular, by
\er{hoyuiouigyfghgjh3478zzrrZZjhjhugyuyughggjhjhjyuuygtyffjjjkjljkjjljokjlllljkkjjkp;;lk}
and
\er{hoyuiouigyfghgjh3478zzrrZZjhjhugyuyughggjhjhjyuuygtyffjjjkjljkjjljokjlllljkkjjkp;;lkkljljklkpkopkpolol;l;;lklklklll}
we obtain:
\begin{multline}\label{hoyuiouigyfghgjh3478zzrrZZjhjhugyuyughggjhjhjyuuygtyffjjjkjljkjjljokjlllljkkjjkp;;lkkljljklkpkopkpolol;l;;lklklkllll;l;kkkl}
\sum_{i=0}^{3}\sum_{j=0}^{3}\hat Z_{ij}\hat
Y^{ij}=\sum_{i=1}^{3}\sum_{j=1}^{3}\frac{1}{c}\left(\left(\frac{\partial
h_i}{\partial x_j}+\frac{\partial h_j}{\partial
x_i}\right)+\frac{\phi}{c}\left(\frac{\partial v^{i}}{\partial
x_{j}}+\frac{\partial v^{j}}{\partial
x_{i}}\right)\right)\left(\frac{\partial v^{i}}{\partial
x_{j}}+\frac{\partial v^{j}}{\partial x_{i}}\right)\\=
\frac{1}{c}\left(d_{\vec x}\vec v+\left\{d_{\vec x}\vec
v\right\}^T\right):\left(d_{\vec x}\vec h+\left\{d_{\vec x}\vec
h\right\}^T\right)+\frac{\phi}{c^2}\left|d_{\vec x}\vec
v+\left\{d_{\vec x}\vec v\right\}^T\right|^2.
\end{multline}
On the other hand by
\er{hoyuiouigyfghgjh3478zzrrZZjhjhugyuyughggjhjhjyuuygtyffjjjkjljkjjljokjljjkjkopopkklklkjjkjkjjkkj}
and
\er{hoyuiouigyfghgjh3478zzrrZZjhjhugyuyughggjhjhjyuuygtyffjjjkjljkjjljokjljjkjkopopkklklkjjkjkjjkkokkjoipio;l;l;ll;lkklkjjk}
we have
\begin{multline}\label{hoyuiouigyfghgjh3478zzrrZZjhjhugyuyughggjhjhjyuuygtyffjjjkjljkjjljokjljjkjkopopkklklkjjkjkjjkkokkjoipio;l;l;ll;lkklkjjkkklk}
\left(\sum_{i=0}^{3}\sum_{j=0}^{3}g^{ij}\hat
Y_{ij}\right)\left(\sum_{i=0}^{3}\sum_{j=0}^{3}g^{ij}\hat
Z_{ij}\right)= \frac{4}{c}\left(div_{\vec x}\vec
v\right)\left(\frac{1}{c}\left(\frac{\partial\phi}{\partial
t}+div_{\vec x}\left\{\phi\vec v\right\}\right)+div_{\vec x}\vec
h\right)\\= \frac{4}{c}\left(div_{\vec x}\vec
h\right)\left(div_{\vec x}\vec
v\right)+\frac{4\phi}{c^2}\left(div_{\vec x}\vec
v\right)^2+\frac{4}{c^2}\left(div_{\vec x}\vec
v\right)\left(\frac{\partial\phi}{\partial t}+\vec
v\cdot\nabla_{\vec x}\phi\right).
\end{multline}
Thus by
\er{hoyuiouigyfghgjh3478zzrrZZjhjhugyuyughggjhjhjyuuygtyffjjjkjljkjjljokjlllljkkjjkp;;lkkljljklkpkopkpolol;l;;lklklkllll;l;kkkl}
and
\er{hoyuiouigyfghgjh3478zzrrZZjhjhugyuyughggjhjhjyuuygtyffjjjkjljkjjljokjljjkjkopopkklklkjjkjkjjkkokkjoipio;l;l;ll;lkklkjjkkklk}
we deduce:
\begin{multline}\label{hoyuiouigyfghgjh3478zzrrZZjhjhugyuyughggjhjhjyuuygtyffjjjkjljkjjljokjljjkjkopopkklklkjjkjkjjkkokkjoipio;l;l;ll;lkklkjjkkklkkkkk}
\left(\sum_{i=0}^{3}\sum_{j=0}^{3}g^{ij}\hat
Y_{ij}\right)\left(\sum_{i=0}^{3}\sum_{j=0}^{3}g^{ij}\hat
Z_{ij}\right)-\sum_{i=0}^{3}\sum_{j=0}^{3}\hat Z_{ij}\hat Y^{ij}
\\= \frac{4}{c}\left(\left(div_{\vec x}\vec h\right)\left(div_{\vec x}\vec
v\right)-\frac{1}{4}\left(d_{\vec x}\vec v+\left\{d_{\vec x}\vec
v\right\}^T\right):\left(d_{\vec x}\vec h+\left\{d_{\vec x}\vec
h\right\}^T\right)\right)\\+\frac{4}{c^2}\left(\left(div_{\vec
x}\vec v\right)\left(\frac{\partial\phi}{\partial t}+\vec
v\cdot\nabla_{\vec x}\phi\right)+\phi\left(\left(div_{\vec x}\vec
v\right)^2-\frac{1}{4}\left|d_{\vec x}\vec v+\left\{d_{\vec x}\vec
v\right\}^T\right|^2\right)\right).
\end{multline}
Note again that by
\er{fgjfjhgghhgjghjhjijhojihjhjjijhjjjjjuiikkjjnj1kkkl} and
\er{fgjfjhgghhgjghjhjkkkkgjghghuiiiulkkjlkklplikklkjljghhgghk} the
four-covector of gravitational potential is a gradient of the global
time multiplied by the constant $c$:
\begin{equation}\label{fgjfjhgghhgjghjhjijhojihjhjjijhjjjjjuiikkjjnj1kkkljkhjjhhjklkk}
(v_0,v_1,v_2,v_3):=\left(1,0,0,0\right)=c\left(\frac{\partial
t}{\partial x^0},\frac{\partial t}{\partial x^1},\frac{\partial
t}{\partial x^2},\frac{\partial t}{\partial x^1}\right).
\end{equation}

Next by \er{hoyuiouigyfghgjh3478zzrrZZjhjhhjkjjhnmnmffjjpopo} and
\er{fgjfjhgghhgjghjhjijhojihjhjjijhjjjjjuiikkjjnj1kkkl} we deduce
that the covariant scalar $\sum_{k=0}^{3}v_kS^k$ satisfies in
cartesian coordinate system:
\begin{equation}\label{ojkkjjkjklkj}
\sum_{k=0}^{3}v_kS^k=\phi.
\end{equation}
Therefore by \er{fgjfjhgghhgjghjhjkkkkgjghghuiiiulkkjlkklplikkl} we
deduce that the four-component field $(d_0,d_1,d_2,d_3)$ defined by:
\begin{equation}\label{fgjfjhgghhgjghjhjkkkkgjghghuiiiulkkjlkklplikkljhjh}
d_j:=\frac{\partial}{\partial
x^j}\left(\sum_{k=0}^{3}v_kS^k\right)
\quad\quad\forall\,j=0,1,2,3,
\end{equation}
is a four-covector. Then the following quantity
\begin{equation}\label{fgjfjhgghhgjghjhjkkkkgjghghuiiiulkkjlkklplikkljhjhlkjkkj}
\sum_{m=0}^{3}\sum_{n=0}^{3}{\Theta}^{mn}d_md_n=\sum_{m=0}^{3}\sum_{n=0}^{3}{\Theta}^{mn}\frac{\partial}{\partial
x^m}\left(\sum_{k=0}^{3}v_kS^k\right)\frac{\partial}{\partial
x^n}\left(\sum_{k=0}^{3}v_kS^k\right)
\end{equation}
is a covariant scalar, where $\{{\Theta}^{ij}\}_{0\leq i,j\leq 3}$
is the contravariant tensor of the three-dimensional geometry that
satisfies
\er{huohuioy89gjjhjffffff3478zzrrZZZhjhhjhhjjhhffGGhjjhuiiuihhjkkoioi}
in every non-inertial cartesian coordinate system:
\begin{equation}\label{huohuioy89gjjhjffffff3478zzrrZZZhjhhjhhjjhhffGGhjjhuiiuihhjkkoioikk}
\begin{cases}
{\Theta}^{00}=0
\\
{\Theta}^{0j}={\Theta}^{j0}=0\quad\quad\forall\, j=1,2,3
\\
{\Theta}^{ij}:=\delta_{ij}\quad\quad\forall\, i,j=1,2,3.
\end{cases}
\end{equation}
Moreover, by \er{hoyuiouigyfghgjh3478zzrrZZffGGjkkjojj} we have:
\begin{equation}\label{hoyuiouigyfghgjh3478zzrrZZffGGjkkjojjlkkl}
g^{ij}:=v^iv^j-{\Theta}^{ij}\quad\quad\forall\,i,j=0,1,2,3,
\end{equation}
On the other hand, inserting \er{ojkkjjkjklkj} and
\er{huohuioy89gjjhjffffff3478zzrrZZZhjhhjhhjjhhffGGhjjhuiiuihhjkkoioikk}
into \er{fgjfjhgghhgjghjhjkkkkgjghghuiiiulkkjlkklplikkljhjhlkjkkj}
we deduce that in the cartesian coordinate system we have:
\begin{multline}\label{fgjfjhgghhgjghjhjkkkkgjghghuiiiulkkjlkklplikkljhjhlkjkkjiooi}
\sum_{m=0}^{3}\sum_{n=0}^{3}{\Theta}^{mn}\frac{\partial}{\partial
x^m}\left(\sum_{j=0}^{3}\sum_{k=0}^{3}g^{kj}v_kS_j\right)\frac{\partial}{\partial
x^n}\left(\sum_{j=0}^{3}\sum_{k=0}^{3}g^{kj}v_kS_j\right)\\=\sum_{m=0}^{3}\sum_{n=0}^{3}{\Theta}^{mn}\frac{\partial}{\partial
x^m}\left(\sum_{k=0}^{3}v_kS^k\right)\frac{\partial}{\partial
x^n}\left(\sum_{k=0}^{3}v_kS^k\right)=\left|\nabla_{\vec
x}\phi\right|^2
\end{multline}
Therefore, by
\er{hoyuiouigyfghgjh3478zzrrZZjhjhugyuyughggjhjhjyuuygtyffjjjkjljkjjljokjljjkjkopopkklklkjjkjkjjkkokkjoipio;l;l;ll;lkklkjjkkklkkkkk}
and
\er{fgjfjhgghhgjghjhjkkkkgjghghuiiiulkkjlkklplikkljhjhlkjkkjiooi} we
deduce that
\begin{multline}\label{hoyuiouigyfghgjh3478zzrrZZjhjhugyuyughggjhjhjyuuygtyffjjjkjljkjjljokjljjkjkopopkklklkjjkjkjjkkokkjoipio;l;l;ll;lkklkjjkkklkkkkkjjkjk}
\sum_{m=0}^{3}\sum_{n=0}^{3}\frac{32\pi
G}{c^4}\,{\Theta}^{mn}\frac{\partial}{\partial
x^m}\left(\sum_{j=0}^{3}\sum_{k=0}^{3}g^{kj}v_kS_j\right)\frac{\partial}{\partial
x^n}\left(\sum_{j=0}^{3}\sum_{k=0}^{3}g^{kj}v_kS_j\right)\\+\left(\sum_{i=0}^{3}\sum_{j=0}^{3}g^{ij}\left(\delta_j
v_i+\delta_i
v_j\right)\right)\left(\sum_{i=0}^{3}\sum_{j=0}^{3}g^{ij}\left(\delta_j
S_i+\delta_i
S_j\right)\right)\\-\sum_{k=0}^{3}\sum_{j=0}^{3}\sum_{m=0}^{3}\sum_{n=0}^{3}\left(\delta_j
S_k+\delta_k S_j\right)g^{km}g^{jn}\left(\delta_m v_n+\delta_n
v_m\right)
\\= \frac{4}{c}\left(\left(div_{\vec x}\vec h\right)\left(div_{\vec x}\vec
v\right)-\frac{1}{4}\left(d_{\vec x}\vec v+\left\{d_{\vec x}\vec
v\right\}^T\right):\left(d_{\vec x}\vec h+\left\{d_{\vec x}\vec
h\right\}^T\right)\right)\\+\frac{4}{c^2}\left(\left(div_{\vec
x}\vec v\right)\left(\frac{\partial\phi}{\partial t}+\vec
v\cdot\nabla_{\vec x}\phi\right)+\phi\left(\left(div_{\vec x}\vec
v\right)^2-\frac{1}{4}\left|d_{\vec x}\vec v+\left\{d_{\vec x}\vec
v\right\}^T\right|^2\right)+\frac{8\pi G}{c^2}\,\left|\nabla_{\vec
x}\phi\right|^2\right),
\end{multline}
where the left hand side of
\er{hoyuiouigyfghgjh3478zzrrZZjhjhugyuyughggjhjhjyuuygtyffjjjkjljkjjljokjljjkjkopopkklklkjjkjkjjkkokkjoipio;l;l;ll;lkklkjjkkklkkkkkjjkjk}
is written in a covariant form. Therefore, for the choice
\begin{equation}\label{vhfffngghkjgghPPNggjgjjkgjoiuighhgghoiuiuouijiihhjioiuojkjk}
\phi=\frac{c^2}{16\pi G}\Phi,\quad\text{and}\quad\vec
h=-\frac{c}{2}\,\vec p:=-\frac{c}{2}(p_1,p_2,p_3)\,.
\end{equation}
we can rewrite
\er{hoyuiouigyfghgjh3478zzrrZZjhjhugyuyughggjhjhjyuuygtyffjjjkjljkjjljokjljjkjkopopkklklkjjkjkjjkkokkjoipio;l;l;ll;lkklkjjkkklkkkkkjjkjk}
in cartesian coordinate system as:
\begin{multline}\label{hoyuiouigyfghgjh3478zzrrZZjhjhugyuyughggjhjhjyuuygtyffjjjkjljkjjljokjljjkjkopopkklklkjjkjkjjkkokkjoipio;l;l;ll;lkklkjjkkklkkkkkjjkjkou}
\frac{1}{2}\left(d_{\vec x}\vec v+\left\{d_{\vec x}\vec
v\right\}^T\right):\left(d_{\vec x}\vec p+\left\{d_{\vec x}\vec
p\right\}^T\right)-2\left(div_{\vec x}\vec v\right)\left(div_{\vec
x}\vec p\right)\\+\frac{1}{4\pi G}\left(div_{\vec x}\vec
v\right)\left(\frac{\partial \Phi}{\partial t}+\vec
v\cdot\nabla_{\vec x} \Phi\right)+\frac{1}{4\pi
G}\Phi\left(div_{\vec x}\vec v\right)^2-\frac{\Phi}{16\pi
G}\left|d_{\vec x}\vec v+\left\{d_{\vec x}\vec
v\right\}^T\right|^2+\frac{1}{8\pi G}\left|\nabla_{\vec
x}\Phi\right|^2=\\
\sum_{m=0}^{3}\sum_{n=0}^{3}\frac{32\pi
G}{c^4}\,{\Theta}^{mn}\frac{\partial}{\partial
x^m}\left(\sum_{j=0}^{3}\sum_{k=0}^{3}g^{kj}v_kS_j\right)\frac{\partial}{\partial
x^n}\left(\sum_{j=0}^{3}\sum_{k=0}^{3}g^{kj}v_kS_j\right)\\+\left(\sum_{i=0}^{3}\sum_{j=0}^{3}g^{ij}\left(\delta_j
v_i+\delta_i
v_j\right)\right)\left(\sum_{i=0}^{3}\sum_{j=0}^{3}g^{ij}\left(\delta_j
S_i+\delta_i
S_j\right)\right)\\-\sum_{k=0}^{3}\sum_{j=0}^{3}\sum_{m=0}^{3}\sum_{n=0}^{3}\left(\delta_j
S_k+\delta_k S_j\right)g^{km}g^{jn}\left(\delta_m v_n+\delta_n
v_m\right),
\end{multline}
where the right hand side is written in the covariant form which is
valid in every curvilinear coordinate system and, due to
\er{hoyuiouigyfghgjh3478zzrrZZjhjhhjkjjhnmnmffjj},
\er{hoyuiouigyfghgjh3478zzrrZZjhjhhjkjjhnmnmffjjkkk} and
\er{vhfffngghkjgghPPNggjgjjkgjoiuighhgghoiuiuouijiihhjioiuojkjk} in
a cartesian coordinate system we have:
\begin{equation}\label{hoyuiouigyfghgjh3478zzrrZZjhjhhjkjjhnmnmffjjhjkjkj}
\begin{cases}
S_0=\frac{c^2}{16\pi G}\Phi-\frac{1}{2}\vec v\cdot\vec p\\
S_j=\frac{c}{2}p_j\quad\quad\forall 1\leq j\leq 3.
\end{cases}
\end{equation}
and
\begin{equation}\label{hoyuiouigyfghgjh3478zzrrZZjhjhhjkjjhnmnmffjjkkkjgjojkjk}
\begin{cases}
S^0=\frac{c^2}{16\pi G}\Phi\\
S^j=\frac{c^2}{16\pi G}\Phi
\frac{v^j}{c}-\frac{c}{2}p_j\quad\quad\forall 1\leq j\leq 3.
\end{cases}
\end{equation}

Next, given a system of $n$ particles with inertial masses
$m_1,\ldots,m_n$, charges $\sigma_1,\ldots,\sigma_n$, places $\vec
r_1(t),\ldots,\vec r_n(t)$ and velocities $\frac{d\vec
r_1}{dt}(t),\ldots, \frac{d\vec r_n}{dt}(t)$ in the cartesian
coordinate system, the usual definitions of the charge density,
current density and the mass density of this system are the
following:
\begin{equation}\label{btfffjhgjghghijhhjhhhjjkjkkj}
\begin{cases}
\rho(\vec
x,t):=\sum\limits_{j=1}^{n}\sigma_j\,\delta^{(3)}\left(\vec x-\vec
r_j(t)\right),\\
\vec j(\vec x,t):=\sum\limits_{j=1}^{n}\sigma_j\frac{d\vec
r_j}{dt}(t)\,\delta^{(3)}\left(\vec x-\vec r_j(t)\right),\\
\mu(\vec x,t):=\sum\limits_{j=1}^{n}m_j\,\delta^{(3)}\left(\vec
x-\vec r_j(t)\right),
\end{cases}
\end{equation}
where $\delta^{(3)}$ is the usual Dirac-delta distribution
(generalized function) in $\mathbb{\R}^3$. Then denoting by
$\delta^{(4)}$  the Dirac-delta distribution in $\mathbb{\R}^4$ and
by $\delta^{(1)}$  the Dirac-delta distribution in $\mathbb{\R}$,
since we have
\begin{multline}\label{btfffjhgjghghijhhjhhhjjkjkkjkkkikkk}
\delta^{(4)}\left(\left(x^0,x^1,x^2,x^3\right)-\left(a^0,a^1,a^2,a^3\right)\right)=\delta^{(3)}\left(\vec
x-\vec a\right)\,\delta^{(1)}\left(x^0-a^0\right)\\
\quad\quad\forall\,\left(x^0,x^1,x^2,x^3\right)=(x^0,\vec
x)\in\mathbb{\R}^4,
\;\;\forall\,\left(a^0,a^1,a^2,a^3\right)=(a^0,\vec
a)\in\mathbb{\R}^4,\\
\text{and}\quad\delta^{(3)}\left(\vec x-\vec
a\right)=\frac{1}{c^3}\delta^{(3)}\left(\frac{1}{c}\left(\vec x-\vec
a\right)\right),
\end{multline}
we rewrite \er{btfffjhgjghghijhhjhhhjjkjkkj} as
\begin{equation}\label{btfffjhgjghghijhhjhhhjjkjkkjjhhh}
\begin{cases}
\rho(\vec
x,t)=\frac{1}{c^3}\int_{\mathbb{R}}\left(\sum\limits_{j=1}^{n}\sigma_j\,\delta^{(4)}\left(\left(x^0,x^1,x^2,x^3\right)
-\left(\chi^0_j(t),\chi^1_j(t),\chi^2_j(t),\chi^3_j(t)\right)\right)\right)dx^0,\\
\vec j(\vec x,t)=
\frac{1}{c^3}\int_{\mathbb{R}}\left(\sum\limits_{j=1}^{n}\sigma_j\frac{d\vec
r_j}{dx^0}(x^0)\,\delta^{(4)}\left(\left(x^0,x^1,x^2,x^3\right)
-\left(\chi^0_j(t),\chi^1_j(t),\chi^2_j(t),\chi^3_j(t)\right)\right)\right)dx^0\\
\mu(\vec
x,t)=\frac{1}{c^3}\int_{\mathbb{R}}\left(\sum\limits_{j=1}^{n}m_j\,\delta^{(4)}\left(\left(x^0,x^1,x^2,x^3\right)
-\left(\chi^0_j(t),\chi^1_j(t),\chi^2_j(t),\chi^3_j(t)\right)\right)\right)dx^0,
\end{cases}
\end{equation}
where we denote
\begin{equation}\label{btfffjhgjghghijhhjhhhjjkjkkhhj}
\left(x^0,x^1,x^2,x^3\right):=\left(t,\frac{1}{c}\vec x\right).
\end{equation}
and
$\left(\chi^0_j(t),\chi^1_j(t),\chi^2_j(t),\chi^3_j(t)\right)\in\mathbb{R}^4$
is a four-dimensional space-time trajectory of the $j$-th particle,
parameterized by the global time, which is defined by the following:
\begin{equation}\label{btfffjhgjghghijhhjhhhjjkjkk}
\left(\chi^0_j(t),\chi^1_j(t),\chi^2_j(t),\chi^3_j(t)\right):=\left(t,\frac{1}{c}\vec
r_j(t)\right).
\end{equation}
Thus, if we denote by $G$ the $4\times 4$-matrix
$G:=\{g_{ij}\}_{0\leq i,j\leq 3}$. then, since the matrix $G$
satisfies $\text{det}\,G=-1$ in every cartesian coordinate system,
we can rewrite the first two equations in
\er{btfffjhgjghghijhhjhhhjjkjkkjjhhh} as:
\begin{multline}\label{btfffjhgjghghijhhjhhhjjkjkkjjhhhjkjjkjhj}
\left(j^0,j^1,j^2,j^3\right):=\left(\rho, \frac{1}{c}\vec j\right)(\vec x,t)=\\
\int_{\mathbb{R}}\left(\sum\limits_{j=1}^{n}\frac{\sigma_j\left(\frac{d\chi^0_j}{dx^0},\frac{d\chi^1_j}{dx^0},\frac{d\chi^2_j}{dx^0},
\frac{d\chi^0_j}{dx^0}\right)(x^0)}{c^3\sqrt{\left|\text{det}\,G\left(x^0,x^1,x^2,x^3\right)\right|}}
\,\delta^{(4)}\left(\left(x^0,x^1,x^2,x^3\right)
-\left(\chi^0_j(t),\chi^1_j(t),\chi^2_j(t),\chi^3_j(t)\right)\right)\right)dx^0.
\end{multline}
Note here that we denoted the matrix $G=\{g_{ij}\}_{0\leq i,j\leq
3}$ by the same letter as the Gravitational Constant $G$. However,
there is no ambiguity, since in the second case $G$ is a constant
scalar and in the first case $G$ is a matrix. Moreover, we will use
the matrix notation $G=\{g_{ij}\}_{0\leq i,j\leq 3}$ only in the
expressions containing term $\text{det}\,G$.

Similarly, we can write:
\begin{multline}\label{btfffjhgjghghijhhjhhhjjkjkkjjhhhjkjjkjhjjkkj}
\mu(\vec x,t)\sqrt{1-\frac{1}{c^2}\left|\vec u(\vec x,t)-\vec v(\vec
\vec x,t)\right|^2}=\\
\int\limits_{\mathbb{R}}\Bigg(\sum\limits_{j=1}^{n}\frac{m_j
\sqrt{\left(\sum\limits_{m=0}^{3}\sum\limits_{k=0}^{3}g_{mk}\left(\left(\chi^0_j,\ldots\chi^3_j\right)\right)\,\frac{d\chi^m_j}{dx^0}\,\frac{d\chi^k_j}{dx^0}\right)(x^0)}
}{c^3\sqrt{\left|\text{det}\,G\left(x^0,\ldots x^3\right)\right|}}
\delta^{(4)}\left(\left(x^0,\ldots x^3\right)
-\left(\chi^0_j,\ldots\chi^3_j\right)(t)\right)\Bigg)dx^0,
\end{multline}
where $\vec u(\vec x,t)$ is the field of velocities of the given
system of particles, considered as a continuum. Next we observe,
that the four-dimensinal quantity in the right hand side of
\er{btfffjhgjghghijhhjhhhjjkjkkjjhhhjkjjkjhj}:
\begin{multline}\label{btfffjhgjghghijhhjhhhjjkjkkjjhhhjkjjkjhjjhhj}
\int_{\mathbb{R}}\left(\sum\limits_{j=1}^{n}\frac{\sigma_j\left(\frac{d\chi^0_j}{dx^0},\frac{d\chi^1_j}{dx^0},\frac{d\chi^2_j}{dx^0},
\frac{d\chi^0_j}{dx^0}\right)(x^0)}{c^3\sqrt{\left|\text{det}\,G\left(x^0,x^1,x^2,x^3\right)\right|}}
\,\delta^{(4)}\left(\left(x^0,x^1,x^2,x^3\right)
-\left(\chi^0_j,\chi^1_j,\chi^2_j,\chi^3_j\right)\right)\right)dx^0
\end{multline}
can be defined also in curvilinear systems of coordinates, where
$\frac{\partial t}{\partial x^0}>0$, and it can be easily proved in
a similar way as it is proved in the General relativity, that this
quantity, defined in every curvilinear coordinate system with
$\frac{\partial t}{\partial x^0}>0$, equals to some
\underline{four-vector}. Similarly, the scalar quantity in the right
hand side of \er{btfffjhgjghghijhhjhhhjjkjkkjjhhhjkjjkjhjjkkj}:
\begin{multline}\label{btfffjhgjghghijhhjhhhjjkjkkjjhhhjkjjkjhjjkkjklkkl}
\int\limits_{\mathbb{R}}\Bigg(\sum\limits_{j=1}^{n}\frac{m_j
\sqrt{\left(\sum\limits_{m=0}^{3}\sum\limits_{k=0}^{3}g_{mk}\left(\left(\chi^0_j,\ldots\chi^3_j\right)\right)\,\frac{d\chi^m_j}{dx^0}\,\frac{d\chi^k_j}{dx^0}\right)(x^0)}
}{c^3\sqrt{\left|\text{det}\,G\left(x^0,\ldots x^3\right)\right|}}
\delta^{(4)}\left(\left(x^0,\ldots x^3\right)
-\left(\chi^0_j,\ldots\chi^3_j\right)\right)\Bigg)dx^0
\end{multline}
can also be defined in curvilinear systems of coordinates and this
quantity forms a covariant \underline{scalar} under the change of
curvilinear coordinate system. On the other hand, since
\er{btfffjhgjghghijhhjhhhjjkjkkjkkkikkk} also valid in every
curvilinear coordinate systems we rewrite
\er{btfffjhgjghghijhhjhhhjjkjkkjjhhhjkjjkjhjjhhj} and
\er{btfffjhgjghghijhhjhhhjjkjkkjjhhhjkjjkjhjjkkjklkkl} as:
\begin{multline}\label{btfffjhgjghghijhhjhhhjjkjkkjjhhhjkjjkjhjjhhjjhj}
\int_{\mathbb{R}}\left(\sum\limits_{j=1}^{n}\frac{\sigma_j\left(\frac{d\chi^0_j}{dx^0},\frac{d\chi^1_j}{dx^0},\frac{d\chi^2_j}{dx^0},
\frac{d\chi^0_j}{dx^0}\right)(x^0)}{c^3\sqrt{\left|\text{det}\,G\left(x^0,x^1,x^2,x^3\right)\right|}}
\,\delta^{(4)}\left(\left(x^0,x^1,x^2,x^3\right)
-\left(\chi^0_j,\chi^1_j,\chi^2_j,\chi^3_j\right)\right)\right)dx^0\\=
\sum\limits_{j=1}^{n}\frac{\sigma_j\left(\frac{d\chi^0_j}{d\chi^0_j},\frac{d\chi^1_j}{d\chi^0_j},\frac{d\chi^2_j}{d\chi^0_j},
\frac{d\chi^0_j}{d\chi^0_j}\right)(\chi^0_j)}{c^3\sqrt{\left|\text{det}\,G\left(\chi^0_j,x^1,x^2,x^3\right)\right|}}
\,\delta^{(3)}\left(\left(x^1,x^2,x^3\right)
-\left(\chi^1_j,\chi^2_j,\chi^3_j\right)(\chi^0_j)\right)\\= \Bigg(
\sum\limits_{j=1}^{n}\frac{\sigma_j}{\sqrt{\left|\text{det}\,G\left(\chi^0_j,x^1,x^2,x^3\right)\right|}}
\,\frac{1}{c^3}\delta^{(3)}\left(\left(x^1,x^2,x^3\right)
-\left(\chi^1_j,\chi^2_j,\chi^3_j\right)(\chi^0_j)\right) ,\\
\sum\limits_{j=1}^{n}\frac{\sigma_j\left(\frac{d\chi^1_j}{dx^0},\frac{d\chi^2_j}{dx^0},
\frac{d\chi^0_j}{dx^0}\right)(\chi^0_j)}{\sqrt{\left|\text{det}\,G\left(\chi^0_j,x^1,x^2,x^3\right)\right|}}
\,\frac{1}{c^3}\delta^{(3)}\left(\left(x^1,x^2,x^3\right)
-\left(\chi^1_j,\chi^2_j,\chi^3_j\right)(\chi^0_j)\right)\Bigg)
\end{multline}
and
\begin{multline}\label{btfffjhgjghghijhhjhhhjjkjkkjjhhhjkjjkjhjjkkjklkklijhjhj}
\int\limits_{\mathbb{R}}\Bigg(\sum\limits_{j=1}^{n}\frac{m_j
\sqrt{\left(\sum\limits_{m=0}^{3}\sum\limits_{k=0}^{3}g_{mk}\left(\left(\chi^0_j,\ldots\chi^3_j\right)\right)\,\frac{d\chi^m_j}{dx^0}\,\frac{d\chi^k_j}{dx^0}\right)(x^0)}
}{c^3\sqrt{\left|\text{det}\,G\left(x^0,\ldots x^3\right)\right|}}
\delta^{(4)}\left(\left(x^0,\ldots x^3\right)
-\left(\chi^0_j,\ldots\chi^3_j\right)\right)\Bigg)dx^0\\=
\sum\limits_{j=1}^{n}\frac{m_j
\sqrt{\left(\sum\limits_{m=0}^{3}\sum\limits_{k=0}^{3}g_{mk}\left(\left(\chi^0_j,\ldots\chi^3_j\right)\right)\,\frac{d\chi^m_j}{d\chi^0_j}\,\frac{d\chi^k_j}{d\chi^0_j}\right)
}
}{c^3\sqrt{\left|\text{det}\,G\left(\chi^0_j,x^1,\ldots
x^3\right)\right|}} \delta^{(3)}\left(\left(x^1,\ldots x^3\right)
-\left(\chi^1_j,\ldots\chi^3_j\right)(\chi^0_j)\right)
\end{multline}
So, by \er{btfffjhgjghghijhhjhhhjjkjkkjjhhhjkjjkjhj} and
\er{btfffjhgjghghijhhjhhhjjkjkkjjhhhjkjjkjhjjhhjjhj} in every
curvilinear coordinate system, where $\frac{\partial t}{\partial
x^0}>0$, we have:
\begin{multline}\label{btfffjhgjghghijhhjhhhjjkjkkjjhhhjkjjkjhjkllkj}
\left(j^0,j^1,j^2,j^3\right):=\frac{1}{\sqrt{\left|\text{det}\,G\right|}}\left(\hat\rho, \frac{1}{c}\hat\rho\,\vec {\hat u}_{x^0}\right)\quad\text{where}\\
\hat\rho:=\sum\limits_{j=1}^{n}\sigma_j \,\delta^{(3)}\left(\vec
{\hat x}
-c\left(\chi^1_j,\chi^2_j,\chi^3_j\right)(\chi^0_j)\right)\\
\quad\text{is the local charge density, calculated in the
curvilinear coordinate system},
\end{multline}
$\vec {\hat x}:=(cx^1,cx^2,cx^3)$ and $\vec{\hat u}_{x^0}$ is the
field of velocities of the system, calculated in a given curvilinear
coordinate system by the differentiation of the last three
coordinates of the particle: $\vec {\hat
r}:=(c\chi^1,c\chi^2,c\chi^3)$ by the coordinate $\chi^0$ that can
be considered as the local time instead of global time $t$. The
quantity in \er{btfffjhgjghghijhhjhhhjjkjkkjjhhhjkjjkjhjkllkj}
equals to a four-vector, under the change of curvilinear coordinate
systems. Similarly by
\er{btfffjhgjghghijhhjhhhjjkjkkjjhhhjkjjkjhjjkkjklkklijhjhj} the
following quantities, defined in every coordinate system where
$\frac{\partial t}{\partial x^0}>0$, equal to some four-vector and
some covariant scalar respectively, under the change of curvilinear
coordinate systems:
\begin{multline}\label{btfffjhgjghghijhhjhhhjjkjkkjjhhhjkjjkjhjjkkjklkklijhjhjkmhggm}
\frac{1}{\sqrt{\left|\text{det}\,G\right|}}\left(\hat\mu,
\frac{1}{c}\hat\mu\,\vec {\hat u}_{x^0}\right)\quad\text{and}\quad
\frac{\hat \mu}{\sqrt{\left|\text{det}\,G\right|}}
\sqrt{\left(\sum\limits_{m=0}^{3}\sum\limits_{k=0}^{3}g_{mk}\,\hat
u^m_{x^0}\,\hat u^k_{x^0}\right)}\quad\quad\quad
\quad\text{where}\\
\hat\mu:=\sum\limits_{j=1}^{n}m_j \,\delta^{(3)}\left(\vec {\hat x}
-c\left(\chi^1_j,\chi^2_j,\chi^3_j\right)(\chi^0_j)\right)\\
\quad\text{is the local mass density, calculated in the curvilinear
coordinate system},
\end{multline}
and $\left(\hat u^0,\hat u^1,\hat u^2,\hat
u^3\right)_{x^0}=\left(1,\frac{1}{c}\vec{\hat u}_{x^0}\right)$ is
the field of four dimensional velocities of the system, calculated
in a given curvilinear coordinate system by the differentiation of
the four dimensional coordinates of the particles by the first
coordinate $\chi^0$. Again note that although the quantities
$\hat\rho$ and $\hat\mu$ are not covariant scalars and $\left(\hat
u^0,\hat u^1,\hat u^2,\hat u^3\right)_{x^0}$ is not a four-vector,
the first quantity in
\er{btfffjhgjghghijhhjhhhjjkjkkjjhhhjkjjkjhjkllkj} equals to a
four-vector and the two first quantities in
\er{btfffjhgjghghijhhjhhhjjkjkkjjhhhjkjjkjhjjkkjklkklijhjhjkmhggm}
equal to a four-vector and a covariant scalar, under the change of
curvilinear coordinate systems. Moreover, clearly the four
dimensional speed $\left(u^0,u^1,u^2,u^3\right)_{t}$, obtained in
curvilinear coordinate system by the differentiation by the global
time $t$, instead of the first local coordinate $\chi_0$, indeed
forms a four-vector and therefore, the quantity, defined in every
coordinate system where $\frac{\partial t}{\partial x^0}>0$,
\begin{equation}\label{btfffjhgjghghijhhjhhhjjkjkkjjhhhjkjjkjhjjkkjklkklijhjhjkmhggmhjhjhkkk}
\frac{\hat \mu}{\sqrt{\left|\text{det}\,G\right|}}
\left(\sum\limits_{m=0}^{3}\sum\limits_{k=0}^{3}g_{mk}\,\hat
u^m_{x^0}\,u^k_{t}\right)\quad\text{where}\;\;
\hat\mu=\sum\limits_{j=1}^{n}m_j \,\delta^{(3)}\left(\vec {\hat x}
-c\left(\chi^1_j,\chi^2_j,\chi^3_j\right)(\chi^0_j)\right)
\end{equation}
equals to a covariant scalar, under the change of curvilinear
coordinate systems.

Next we can write the density of the Lagrangian of the
electromagnetic field, defined in \er{vhfffngghkjgghPPNggjgjjkgj} in
the equivalent form \er{vhfffngghkjgghPPNggjgjjkgjoiuioiggh}, where
the right hand side is written in a covariant form which is valid
for every curvilinear coordinate system:
\begin{multline}\label{vhfffngghkjgghPPNggjgjjkgjoiuioigghghghfguui}
\frac{1}{4\pi}\left(\frac{1}{2}\left|\vec
D\right|^2-\frac{1}{2}\left|\vec
B\right|^2-4\pi\left(\rho\Psi-\frac{1}{c}\vec A\cdot\vec
j\right)\right)=\\
\frac{1}{4\pi}\left(-\sum_{n=0}^{3}\sum_{k=0}^{3}\sum_{m=0}^{3}\sum_{p=0}^{3}\frac{1}{4}g^{mn}g^{pk}\left(\frac{\partial
A_p}{\partial x^m}-\frac{\partial A_m}{\partial
x^p}\right)\left(\frac{\partial A_k}{\partial x^n}-\frac{\partial
A_n}{\partial x^k}\right)-\sum_{k=0}^{3}4\pi j^k A_k\right),
\end{multline}
where $(j^0,j^1,j^2,j^3)$ is the four-vector of the current that
satisfies \er{btfffjhgjghghijhhjhhhjjkjkkjjhhhjkjjkjhjkllkj} in
every curvilinear coordinate system, where $\frac{\partial
t}{\partial x^0}>0$, and $(A_0,A_1,A_2,A_3)$ is the four-covector of
the electromagnetic potential. Then by
\er{vhfffngghkjgghPPNggjgjjkgjoiuioigghghghfguui},
\er{hoyuiouigyfghgjh3478zzrrZZjhjhugyuyughggjhjhjyuuygtyffjjjkjljkjjljokjljjkjkopopkklklkjjkjkjjkkokkjoipio;l;l;ll;lkklkjjkkklkkkkkjjkjkou}
and
\er{btfffjhgjghghijhhjhhhjjkjkkjjhhhjkjjkjhjjkkjklkklijhjhjkmhggmhjhjhkkk}
we rewrite the Lagrangian density in \er{vhfffngghkjgghDDmmioiuiu}
in the cartesian coordinate system as:
\begin{multline}\label{vhfffngghkjgghDDmmioiuiujhhjhj}
\frac{1}{8\pi}\left|-\nabla_{\vec
x}\Psi-\frac{1}{c}\frac{\partial\vec A}{\partial t}+\frac{1}{c}\vec
v\times curl_{\vec x}\vec A\right|^2-\frac{1}{8\pi}\left|curl_{\vec
x}\vec A\right|^2-\left(\rho\Psi-\frac{1}{c}\vec A\cdot\vec
j\right)\\+\mu g\left(\frac{1}{2}\left|\vec u-\vec
v\right|^2\right)+
\frac{1}{2}\left(d_{\vec x}\vec v+\left\{d_{\vec x}\vec
v\right\}^T\right):\left(d_{\vec x}\vec p+\left\{d_{\vec x}\vec
p\right\}^T\right)-2\left(div_{\vec x}\vec v\right)\left(div_{\vec
x}\vec p\right)\\+\frac{1}{4\pi G}\left(div_{\vec x}\vec
v\right)\left(\frac{\partial \Phi}{\partial t}+\vec
v\cdot\nabla_{\vec x} \Phi\right)+\frac{1}{4\pi
G}\Phi\left(div_{\vec x}\vec v\right)^2-\frac{\Phi}{16\pi
G}\left|d_{\vec x}\vec v+\left\{d_{\vec x}\vec
v\right\}^T\right|^2+\frac{1}{8\pi G}\left|\nabla_{\vec
x}\Phi\right|^2=L_{ge}=
\\
\frac{1}{4\pi}\left(-\sum_{n=0}^{3}\sum_{k=0}^{3}\sum_{m=0}^{3}\sum_{p=0}^{3}\frac{1}{4}g^{mn}g^{pk}\left(\frac{\partial
A_p}{\partial x^m}-\frac{\partial A_m}{\partial
x^p}\right)\left(\frac{\partial A_k}{\partial x^n}-\frac{\partial
A_n}{\partial x^k}\right)-\sum_{k=0}^{3}4\pi j^k A_k\right)\\
+\frac{\hat \mu}{\sqrt{\left|\text{det}\,G\right|}}
\left(\sum\limits_{m=0}^{3}\sum\limits_{k=0}^{3}g_{mk}\,\hat
u^m_{x^0}\,u^k_{t}\right)\left(\sum\limits_{m=0}^{3}\sum\limits_{k=0}^{3}g_{mk}\,
u^m_{t}\,u^k_{t}\right)^{-1}g\left(\frac{c^2}{2}-\sum\limits_{m=0}^{3}\sum\limits_{k=0}^{3}\frac{c^2}{2}g_{mk}\,
u^m_{t}\,u^k_{t}\right)
\\
+\sum_{m=0}^{3}\sum_{n=0}^{3}\frac{32\pi
G}{c^4}\,{\Theta}^{mn}\frac{\partial}{\partial
x^m}\left(\sum_{j=0}^{3}\sum_{k=0}^{3}g^{kj}v_kS_j\right)\frac{\partial}{\partial
x^n}\left(\sum_{j=0}^{3}\sum_{k=0}^{3}g^{kj}v_kS_j\right)\\+\left(\sum_{i=0}^{3}\sum_{j=0}^{3}g^{ij}\left(\delta_j
v_i+\delta_i
v_j\right)\right)\left(\sum_{i=0}^{3}\sum_{j=0}^{3}g^{ij}\left(\delta_j
S_i+\delta_i
S_j\right)\right)\\-\sum_{k=0}^{3}\sum_{j=0}^{3}\sum_{m=0}^{3}\sum_{n=0}^{3}\left(\delta_j
S_k+\delta_k S_j\right)g^{km}g^{jn}\left(\delta_m v_n+\delta_n
v_m\right),
\end{multline}
where the right hand side is written in the covariant form and the
second equality is valid in every curvilinear coordinate system with
$\frac{\partial t}{\partial x^0}>0$, $(j^0,j^1,j^2,j^3)$ is the
four-vector of the current, that satisfies
\er{btfffjhgjghghijhhjhhhjjkjkkjjhhhjkjjkjhjkllkj}, $\hat\mu$ is
given by
\er{btfffjhgjghghijhhjhhhjjkjkkjjhhhjkjjkjhjjkkjklkklijhjhjkmhggmhjhjhkkk}
and in a cartesian coordinate system we have:
\begin{equation}\label{hoyuiouigyfghgjh3478zzrrZZjhjhhjkjjhnmnmffjjhjkjkjjhhj}
\begin{cases}
S_0=\frac{c^2}{16\pi G}\Phi-\frac{1}{2}\vec v\cdot\vec p\\
S_j=\frac{c}{2}p_j\quad\quad\forall 1\leq j\leq 3.
\end{cases}
\end{equation}
and
\begin{equation}\label{hoyuiouigyfghgjh3478zzrrZZjhjhhjkjjhnmnmffjjkkkjgjojkjkkjjh}
\begin{cases}
S^0=\frac{c^2}{16\pi G}\Phi\\
S^j=\frac{c^2}{16\pi G}\Phi
\frac{v^j}{c}-\frac{c}{2}p_j\quad\quad\forall 1\leq j\leq 3.
\end{cases}
\end{equation}
Moreover in the particular case of the relativistic-like choice
$g(s):=-c^2\sqrt{1-\frac{2s}{c^2}}$, by
\er{btfffjhgjghghijhhjhhhjjkjkkjjhhhjkjjkjhjjkkjklkklijhjhjkmhggm}
we can write an alternative to \er{vhfffngghkjgghDDmmioiuiujhhjhj}
as:
\begin{multline}\label{vhfffngghkjgghDDmmioiuiujhhjhjjkjk}
\frac{1}{8\pi}\left|-\nabla_{\vec
x}\Psi-\frac{1}{c}\frac{\partial\vec A}{\partial t}+\frac{1}{c}\vec
v\times curl_{\vec x}\vec A\right|^2-\frac{1}{8\pi}\left|curl_{\vec
x}\vec A\right|^2-\left(\rho\Psi-\frac{1}{c}\vec A\cdot\vec
j\right)\\+\mu g\left(\frac{1}{2}\left|\vec u-\vec
v\right|^2\right)+
\frac{1}{2}\left(d_{\vec x}\vec v+\left\{d_{\vec x}\vec
v\right\}^T\right):\left(d_{\vec x}\vec p+\left\{d_{\vec x}\vec
p\right\}^T\right)-2\left(div_{\vec x}\vec v\right)\left(div_{\vec
x}\vec p\right)\\+\frac{1}{4\pi G}\left(div_{\vec x}\vec
v\right)\left(\frac{\partial \Phi}{\partial t}+\vec
v\cdot\nabla_{\vec x} \Phi\right)+\frac{1}{4\pi
G}\Phi\left(div_{\vec x}\vec v\right)^2-\frac{\Phi}{16\pi
G}\left|d_{\vec x}\vec v+\left\{d_{\vec x}\vec
v\right\}^T\right|^2+\frac{1}{8\pi G}\left|\nabla_{\vec
x}\Phi\right|^2=L_{ge}=\\
\frac{1}{4\pi}\left(-\sum_{n=0}^{3}\sum_{k=0}^{3}\sum_{m=0}^{3}\sum_{p=0}^{3}\frac{1}{4}g^{mn}g^{pk}\left(\frac{\partial
A_p}{\partial x^m}-\frac{\partial A_m}{\partial
x^p}\right)\left(\frac{\partial A_k}{\partial x^n}-\frac{\partial
A_n}{\partial x^k}\right)-\sum_{k=0}^{3}4\pi j^k A_k\right)\\
-\frac{c^2\hat \mu}{\sqrt{\left|\text{det}\,G\right|}}
\sqrt{\left(\sum\limits_{m=0}^{3}\sum\limits_{k=0}^{3}g_{mk}\,\hat
u^m_{x^0}\,\hat u^k_{x^0}\right)}
\\+\sum_{m=0}^{3}\sum_{n=0}^{3}\frac{32\pi
G}{c^4}\,{\Theta}^{mn}\frac{\partial}{\partial
x^m}\left(\sum_{j=0}^{3}\sum_{k=0}^{3}g^{kj}v_kS_j\right)\frac{\partial}{\partial
x^n}\left(\sum_{j=0}^{3}\sum_{k=0}^{3}g^{kj}v_kS_j\right)\\+\left(\sum_{i=0}^{3}\sum_{j=0}^{3}g^{ij}\left(\delta_j
v_i+\delta_i
v_j\right)\right)\left(\sum_{i=0}^{3}\sum_{j=0}^{3}g^{ij}\left(\delta_j
S_i+\delta_i
S_j\right)\right)\\-\sum_{k=0}^{3}\sum_{j=0}^{3}\sum_{m=0}^{3}\sum_{n=0}^{3}\left(\delta_j
S_k+\delta_k S_j\right)g^{km}g^{jn}\left(\delta_m v_n+\delta_n
v_m\right),
\end{multline}
where the right hand side is written in the covariant form and the
second equality is valid in every curvilinear coordinate system with
$\frac{\partial t}{\partial x^0}>0$. On the other hand, in the case
of fully non-relativistic Lagrangian, where
$g(s)=\left(s-\frac{c^2}{2}\right)$, we can write an alternative to
\er{vhfffngghkjgghDDmmioiuiujhhjhj} as:
\begin{multline}\label{vhfffngghkjgghDDmmioiuiujhhjhjklkkl}
\frac{1}{8\pi}\left|-\nabla_{\vec
x}\Psi-\frac{1}{c}\frac{\partial\vec A}{\partial t}+\frac{1}{c}\vec
v\times curl_{\vec x}\vec A\right|^2-\frac{1}{8\pi}\left|curl_{\vec
x}\vec A\right|^2-\left(\rho\Psi-\frac{1}{c}\vec A\cdot\vec
j\right)\\+\mu g\left(\frac{1}{2}\left|\vec u-\vec
v\right|^2\right)+
\frac{1}{2}\left(d_{\vec x}\vec v+\left\{d_{\vec x}\vec
v\right\}^T\right):\left(d_{\vec x}\vec p+\left\{d_{\vec x}\vec
p\right\}^T\right)-2\left(div_{\vec x}\vec v\right)\left(div_{\vec
x}\vec p\right)\\+\frac{1}{4\pi G}\left(div_{\vec x}\vec
v\right)\left(\frac{\partial \Phi}{\partial t}+\vec
v\cdot\nabla_{\vec x} \Phi\right)+\frac{1}{4\pi
G}\Phi\left(div_{\vec x}\vec v\right)^2-\frac{\Phi}{16\pi
G}\left|d_{\vec x}\vec v+\left\{d_{\vec x}\vec
v\right\}^T\right|^2+\frac{1}{8\pi G}\left|\nabla_{\vec
x}\Phi\right|^2=L_{ge}=\\
\frac{1}{4\pi}\left(-\sum_{n=0}^{3}\sum_{k=0}^{3}\sum_{m=0}^{3}\sum_{p=0}^{3}\frac{1}{4}g^{mn}g^{pk}\left(\frac{\partial
A_p}{\partial x^m}-\frac{\partial A_m}{\partial
x^p}\right)\left(\frac{\partial A_k}{\partial x^n}-\frac{\partial
A_n}{\partial x^k}\right)-\sum_{k=0}^{3}4\pi j^k A_k\right)\\
-\frac{c^2\hat \mu}{\sqrt{\left|\text{det}\,G\right|}}
\left(\sum\limits_{m=0}^{3}\sum\limits_{k=0}^{3}g_{mk}\,\hat
u^m_{x^0}\,u^k_{t}\right)
\\+\sum_{m=0}^{3}\sum_{n=0}^{3}\frac{32\pi
G}{c^4}\,{\Theta}^{mn}\frac{\partial}{\partial
x^m}\left(\sum_{j=0}^{3}\sum_{k=0}^{3}g^{kj}v_kS_j\right)\frac{\partial}{\partial
x^n}\left(\sum_{j=0}^{3}\sum_{k=0}^{3}g^{kj}v_kS_j\right)\\+\left(\sum_{i=0}^{3}\sum_{j=0}^{3}g^{ij}\left(\delta_j
v_i+\delta_i
v_j\right)\right)\left(\sum_{i=0}^{3}\sum_{j=0}^{3}g^{ij}\left(\delta_j
S_i+\delta_i
S_j\right)\right)\\-\sum_{k=0}^{3}\sum_{j=0}^{3}\sum_{m=0}^{3}\sum_{n=0}^{3}\left(\delta_j
S_k+\delta_k S_j\right)g^{km}g^{jn}\left(\delta_m v_n+\delta_n
v_m\right),
\end{multline}
where again the right hand side is written in the covariant form and
the second equality is valid in every curvilinear coordinate system
with $\frac{\partial t}{\partial x^0}>0$.

Once we wrote the Lagrangian density
$L_{ge}:=L\left(S^k,v^k,A^k,x^k\right)_{k=0,\ldots 4}$ as a
covariant scalar, under the changes of curvilinear coordinate
systems, we can write a covariant Lagrangian as:
\begin{equation}\label{uuuuuuuu}
J_{ge}(S^k,v^k,A^k):=\int_{(x^0,x^1,x^2,x^3)}L_{ge}\left(S^k,v^k,A^k,x^k\right)\,\sqrt{\left|\text{det}\,G\right|}\,dx^0dx^1dx^2dx^3.
\end{equation}
Although we need a term $\sqrt{\left|\text{det}\,G\right|}$ for the
covariance of the Lagrangian in curvilinear coordinate systems, in
the cartesian coordinate systems we always have
$\sqrt{\left|\text{det}\,G\right|}=1$.

Next note that the contravariant tensor of the three-dimensional
geometry $\Theta^{ij}$ which satisfies
\er{huohuioy89gjjhjffffff3478zzrrZZZhjhhjhhjjhhffGGhjjhuiiuihhjkkoioikk}
in non-inertial cartesian coordinate systems and the scalar of the
global time $t$ are dependent on the geometry of the space-time only
and are fully determined in a given curvilinear coordinate system by
change of variables rules. In particular, the
four-\underline{covector} of the gravitational potential
$(v_0,v_1,v_2,v_3)$ is fully determined in the given curvilinear
coordinate system, since we have:
\begin{equation}\label{fgjfjhfgfgfgfhihhjhhjjhhjhjjhjhhjhjhjhjhjhhjhj}
v_k=c\frac{\partial t}{\partial x^k}\quad\quad\forall\, k=0,1,2,3.
\end{equation}
Moreover, by
\er{fgjfjhgghhgjghjhjijhojihjhjjijhjjjjjuiijjjkhjhjhjuiiuuuyuuoiuuiio}
and \er{fgjfjhfgfgfgfhihhjhhjjhhjhj} we have the following covariant
identities which are valid in every curvilinear coordinate system:
\begin{equation}\label{fgjfjhfgfgfgfhihhjhhjjhhjhjjhjhhjhj}
\begin{cases}
\sum_{k=0}^{3}{\Theta}^{mk}v_k=\sum_{k=0}^{3}c\,{\Theta}^{mk}\frac{\partial
t}{\partial x^k}=0\quad\quad\forall\, m=0,1,2,3\\
\sum_{k=0}^{3}\sum_{j=0}^{3}g^{kj}v_kv_j=\sum_{k=0}^{3}\sum_{j=0}^{3}c^2g^{kj}\frac{\partial
t}{\partial x^k}\frac{\partial t}{\partial x^j}=1.
\end{cases}
\end{equation}
However the four-\underline{vector} of the gravitational potential
$(v^0,v^1,v^2,v^3)$, the contravariant pseudometric tensor
$g^{mn}=v^mv^n-\Theta^{mn}$ and thus also the covariant pseudometric
tensor $g_{mn}$ depend also on the physical properties of the
matter. Without knowing the physical properties of the matter the
four-vector of the gravitational potential can be arbitrary vector
$(v^0,v^1,v^2,v^3)$ that satisfies:
\begin{equation}\label{fgjfjhfgfgfgfhihhjhhjjhhjhjjhjhhjhjkjh}
\sum_{k=0}^{3}v_kv^k=\sum_{k=0}^{3}c\frac{\partial t}{\partial
x^k}v^k
=1.
\end{equation}
Indeed for an arbitrary four-vector $(v^0,v^1,v^2,v^3)$ that
satisfies \er{fgjfjhfgfgfgfhihhjhhjjhhjhjjhjhhjhjkjh}, denoting
$g^{mn}:=v^mv^n-\Theta^{mn}$, using
\er{fgjfjhfgfgfgfhihhjhhjjhhjhjjhjhhjhj} and
\er{fgjfjhfgfgfgfhihhjhhjjhhjhjjhjhhjhjkjh} we clearly obtain the
following consistency:
\begin{equation}\label{fgjfjhfgfgfgfhihhjhhjjhhjhjjhjhhjhjkjhjjkjkk}
\sum_{j=0}^{3}g^{kj}v_j=\sum_{j=0}^{3}\left(v^kv^j-\Theta^{kj}\right)v_j=v^k\left(\sum_{j=0}^{3}v^jv_j\right)-\sum_{j=0}^{3}\Theta^{kj}v_j=v^k\quad\quad\forall
k=0,1,2,3.
\end{equation}
Thus we obtained that the four-vector of the gravitational potential
can be arbitrary four vector in \er{uuuuuuuu} that satisfies the
linear constraint \er{fgjfjhfgfgfgfhihhjhhjjhhjhjjhjhhjhjkjh} where
the four-covector $v_k$ is prescribed. So the four-vector
$(v^0,v^1,v^2,v^3)$ actually contains three independent scalar
functions similarly as in cartesian coordinate systems where we have
$v^0=1$. On the other hand, the four-vector $S^k$ contains four
independent scalar functions. Thus the Lagrangian in \er{uuuuuuuu}
depends on seven independent scalar functions characterizing the
gravitational field and the four-vector of electromagnetic
potential, exactly as in cartesian coordinate systems where we have
four independent scalar functions that characterize the
electromagnetic field: scalar $\Psi$ and three-dimensional vector
$\vec A$ and seven independent scalar functions that characterize
the gravitational field: three are contained in the
three-dimensional vectorial gravitational potential $\vec v$ and
other four are the ancillary scalar field $\Phi$ and the ancillary
three-dimensional vector field $\vec p$.

Finally note that, since our model of the Newtonian-type gravity in
the case of fully non-relativistic choice
\er{vhfffngghkjgghDDmmioiuiujhhjhjklkkl} and in the absence of
electromagnetic fields coincides with the classical Newtonian
gravitation, as a particular case, we obtained a covariant
formulation of the classical Newtonian gravity in curvilinear
coordinate systems.

\subsubsection{The case of some alternative model of the gravity}\label{gggggyffu997tg7gffttft999}

Consider $\vec k$ to be the vectorial potential of the inertia,
which is a generally trivial speed-like vector field, assumed to be
fixed in every fixed inertial or non-inertial cartesian coordinate
system (see Definition \ref{jjhgyftdtd}). In this subsection we find
a equivalent form of the Lagrangian density of the
gravitational-electromagnetic field in the case of the alternative
model of the gravity gravity having the general form
\er{vhfffngghkjgghDDnw22mm34} or
\er{vhfffngghkjgghDDnw22jkjkjkjk22gughhg22mm34}:
\begin{multline}\label{vhfffngghkjgghDDnw22mm3488}
L\left(\vec A,\Psi,\vec v,\Phi_0,\vec h,\vec x,t\right):=\\
\frac{1}{8\pi}\left|-\nabla_{\vec
x}\Psi-\frac{1}{c}\frac{\partial\vec A}{\partial t}+\frac{1}{c}\vec
v\times curl_{\vec x}\vec A\right|^2-\frac{1}{8\pi}\left|curl_{\vec
x}\vec A\right|^2-\left(\rho\Psi-\frac{1}{c}\vec A\cdot\vec j\right)
+\mu g\left(\frac{1}{2}\left|\vec u-\vec v\right|^2\right)
\\
-\frac{c^2}{8\pi G}\left|-\nabla_{\vec x}\Phi_0
-\frac{1}{c}\frac{\partial\vec h}{\partial t}
+\frac{1}{c}\vec v\times curl_{\vec x}\vec h-\frac{1}{c}\nabla_{\vec
x}\left(\vec h\cdot\vec v\right)\right|^2+\frac{c^2}{8\pi
G}\left|curl_{\vec x}\vec h\right|^2\\+\frac{c^2\beta}{4\pi
G}\left(\frac{1}{4}\left|d_{\vec x}\vec v+\left\{d_{\vec x}\vec
v\right\}^T\right|^2-\left|\Div_{\vec x}\vec v\right|^2\right),
\end{multline}
where
\begin{equation}\label{btfffygtgyggyDDnw22jkhk22mm3488}
\vec h=\vec v-\vec
k\quad\text{and}\quad\Phi_0=-\frac{1}{2c}\left|\vec v-\vec
k\right|^2,
\end{equation}
so that we have
\begin{multline}\label{vhfffngghkjgghDDnw22jkjkjkjk22gughhg22mm3488}
L\left(\vec A,\Psi,\vec v,\Phi_0,\vec h,\vec
x,t\right)=L_1\left(\vec A,\Psi,\vec v,\vec
x,t\right):=\\
\frac{1}{8\pi}\left|-\nabla_{\vec
x}\Psi-\frac{1}{c}\frac{\partial\vec A}{\partial t}+\frac{1}{c}\vec
v\times curl_{\vec x}\vec A\right|^2-\frac{1}{8\pi}\left|curl_{\vec
x}\vec A\right|^2-\left(\rho\Psi-\frac{1}{c}\vec A\cdot\vec
j\right)+\mu g\left(\frac{1}{2}\left|\vec u-\vec v\right|^2\right)
\\
-\frac{c^2}{8\pi G}\left|
\frac{1}{c}\nabla_{\vec x}\left(\frac{1}{2}\left|\vec v-\vec
k\right|^2\right)-\frac{1}{c}\frac{\partial}{\partial t}\left(\vec
v-\vec k\right)
+\frac{1}{c}\vec v\times curl_{\vec x}\left(\vec v-\vec
k\right)-\frac{1}{c}\nabla_{\vec x}\left(\left(\vec v-\vec
k\right)\cdot\vec v\right)\right|^2\\+\frac{c^2}{8\pi
G}\left|curl_{\vec x}\left(\vec v-\vec
k\right)\right|^2+\frac{c^2\beta}{4\pi
G}\left(\frac{1}{4}\left|d_{\vec x}\vec v+\left\{d_{\vec x}\vec
v\right\}^T\right|^2-\left|\Div_{\vec x}\vec v\right|^2\right).
\end{multline}
Our purpose is to make the equivalent form of this Lagrangian
density to be covariant and valid in every \underline{curvilinear}
coordinate system. Assume first, that we deal with a cartesian
inertial or non-inertial coordinate system. Consider again the
three-dimensional vectorial gravitational potential $\vec
v=(v^1,v^2,v^3)$ and consider the covariant pseudometric tensor
$G=\{g_{ij}\}_{0\leq i,j\leq 3}$ defined by
\er{hoyuiouigyfg3478zzrrZZff} and the contravariant pseudometric
tensor $\tilde G=\{g^{ij}\}_{0\leq i,j\leq 3}$ defined by
\er{hoyuiouigyfghgjh3478zzrrZZff}. Next, as before, we consider the
four-dimensional gravitational potential $(v^0,v^1,v^2,v^3)$ defined
by \er{fgjfjhgghhgjghjhjijhojihjhjjijhjjjjjuiijjjkhjhjhjklkk} and
the corresponding lowered four-covector field $(v_0,v_1,v_2,v_3)$,
that we called the four-covector of gravitational potential, defined
by \er{fgjfjhgghhgjghjhjijhojihjhjjijhjjjjjuiikkjjnj1kkkl} and, as
in \er{fgjfjhgghhghjjgj}, denote
\begin{equation}\label{fgjfjhgghhghjjgjhkkhjhjh}
\hat Y_{ij}:=\delta_j v_i+\delta_i v_j\quad\quad\forall\, 0\leq
i,j\leq 3,
\end{equation}
where $\delta_j$ means the covariant derivative with respect to the
dynamical pseudo-metrics $G=\{g_{ij}\}_{0\leq i,j\leq 3}$. Then by
the particular case of
\er{hoyuiouigyfghgjh3478zzrrZZjhjhugyuyughggjhjhjyuuygtyffjjjkjljkjjljokjljjkjkopopkklklkjjkjkjjkkokkjoipio;l;l;ll;lkklkjjkkklkkkkk}
with $\hat Z_{ij}=\hat Y_{ij}$ we have
\begin{multline}\label{hoyuiouigyfghgjh3478zzrrZZjhjhugyuyughggjhjhjyuuygtyffjjjkjljkjjljokjljjkjkopopkklklkjjkjkjjkkokkjoipio;l;l;ll;lkklkjjkkklkkkkkjojjh}
\left(\left(div_{\vec x}\vec v\right)^2-\frac{1}{4}\left|d_{\vec
x}\vec v+\left\{d_{\vec x}\vec v\right\}^T\right|^2\right)=\\
\frac{c^2}{4}\left(\left(\sum_{i=0}^{3}\sum_{j=0}^{3}g^{ij}\left(\delta_j
v_i+\delta_i
v_j\right)\right)^2-\sum_{k=0}^{3}\sum_{j=0}^{3}\sum_{m=0}^{3}\sum_{n=0}^{3}\left(\delta_j
v_k+\delta_k v_j\right)g^{km}g^{jn}\left(\delta_m v_n+\delta_n
v_m\right)\right)
.
\end{multline}
Furthermore, consider the Dynamical four-covector of genuine gravity
$(s_0,s_1,s_2,s_3)$, defined as in
\er{fgjfjhgghhgjghjhjijhojihjhjjijhjjjjjuiijjjk88} by:
\begin{multline}\label{fgjfjhgghhgjghjhjijhojihjhjjijhjjjjjuiijjjk88jihhhi}
(s_0,s_1,s_2,s_3)
=-\frac{1}{c}\left(\left(\Phi_0+\frac{1}{c}\vec v\cdot\vec
h\right),-\vec h\right)=\left(\frac{1}{2c^2}|\vec
h|^2-\frac{1}{c^2}\vec v\cdot\vec h\,,\,\frac{1}{c}\vec
h\right)\quad\quad\text{where}\quad\\ s_0
=-\frac{1}{c}\left(\Phi_0+\frac{1}{c}\vec v\cdot\vec
h\right)=\frac{1}{2c^2}|\vec h|^2-\frac{1}{c^2}\vec v\cdot\vec
h\;\;\quad\text{and}\quad\;\;(s_1,s_2,s_3)=\frac{1}{c}\vec h,
\end{multline}
where $\vec h$ and $\Phi_0$ are given by
\er{btfffygtgyggyDDnw22jkhk22mm3488}. Next, as in
\er{MaxVacFullPPNhjjghjjkjhhiuyy} and
\er{vhfffngghkjgghPPNggjgjjkgjoiuioiggh} we have the following:
\begin{multline}\label{vhfffngghkjgghPPNggjgjjkgjoiuijouuiyiyyyuyu}
\frac{1}{4\pi}\left(\frac{1}{2}\left|-\nabla_{\vec
x}\Psi-\frac{1}{c}\frac{\partial\vec A}{\partial t}+\frac{1}{c}\vec
v\times curl_{\vec x}\vec A\right|^2-\frac{1}{2}\left|curl_{\vec
x}\vec A\right|^2-4\pi\left(\rho\Psi-\frac{1}{c}\vec A\cdot\vec
j\right)\right)=\\
\frac{1}{4\pi}\left(-\sum_{n=0}^{3}\sum_{k=0}^{3}\sum_{m=0}^{3}\sum_{p=0}^{3}\frac{1}{4}g^{mn}g^{pk}\left(\frac{\partial
A_p}{\partial x^m}-\frac{\partial A_m}{\partial
x^p}\right)\left(\frac{\partial A_k}{\partial x^n}-\frac{\partial
A_n}{\partial x^k}\right)-\sum_{k=0}^{3}4\pi j^k A_k\right),
\end{multline}
where $(A_0,A_1,A_2,A_3)$ is the four-covector of thefour
dimensional electromagnetic potential, defined as in
\er{fgjfjhgghhgjghjhjijhojihjhjjijhjjjljljpk} by:
\begin{equation}\label{fgjfjhgghhgjghjhjijhojihjhjjijhjjjljljpkjkjkjoj}
(A_0,A_1,A_2,A_3)=(\Psi,-\vec A)\quad\text{where}\quad
A_0=\Psi\;\;\text{and}\;\;(A_1,A_2,A_3)=-\vec A,
\end{equation}
and $(j^0,j^1,j^2,j^3)$ is the four-vector of the current, given by
\er{btfffjhgjghghijhhjhhhjjkjkkjjhhhjkjjkjhjkllkj}. Then, completely
analogously, as we get in \er{MaxVacFullPPNhjjghjjkjhhiuyy} and
\er{vhfffngghkjgghPPNggjgjjkgjoiuioiggh} the following part of
\er{vhfffngghkjgghPPNggjgjjkgjoiuijouuiyiyyyuyu}:
\begin{multline}\label{vhfffngghkjgghPPNggjgjjkgjoiuipopiioio}
\left(\frac{1}{2}\left|-\nabla_{\vec
x}\Psi-\frac{1}{c}\frac{\partial\vec A}{\partial t}+\frac{1}{c}\vec
v\times curl_{\vec x}\vec A\right|^2-\frac{1}{2}\left|curl_{\vec
x}\vec A\right|^2\right)=\\
-\sum_{n=0}^{3}\sum_{k=0}^{3}\sum_{m=0}^{3}\sum_{p=0}^{3}\frac{1}{4}g^{mn}g^{pk}\left(\frac{\partial
A_p}{\partial x^m}-\frac{\partial A_m}{\partial
x^p}\right)\left(\frac{\partial A_k}{\partial x^n}-\frac{\partial
A_n}{\partial x^k}\right),
\end{multline}
we can get also:
\begin{multline}\label{vhfffngghkjgghPPNggjgjjkgjoiuipopiioioouuou}
\left(\frac{1}{2}\left|-\nabla_{\vec x}\Phi_0
-\frac{1}{c}\frac{\partial\vec h}{\partial t}
+\frac{1}{c}\vec v\times curl_{\vec x}\vec h-\frac{1}{c}\nabla_{\vec
x}\left(\vec h\cdot\vec
v\right)\right|^2-\frac{1}{2}\left|curl_{\vec
x}\vec h\right|^2\right)=\\
\left(\frac{1}{2}\left|-\nabla_{\vec x}\left(\Phi_0+\frac{1}{c}\vec
v\cdot\vec h\right) -\frac{1}{c}\frac{\partial\vec h}{\partial t}
+\frac{1}{c}\vec v\times curl_{\vec x}\vec
h\right|^2-\frac{1}{2}\left|curl_{\vec x}\vec
h\right|^2\right)=\\
-\sum_{n=0}^{3}\sum_{k=0}^{3}\sum_{m=0}^{3}\sum_{p=0}^{3}\frac{c^2}{4}g^{mn}g^{pk}\left(\frac{\partial
s_p}{\partial x^m}-\frac{\partial s_m}{\partial
x^p}\right)\left(\frac{\partial s_k}{\partial x^n}-\frac{\partial
s_n}{\partial x^k}\right).
\end{multline}
Then, similarly as it was done in
\er{vhfffngghkjgghDDmmioiuiujhhjhj}, by
\er{vhfffngghkjgghPPNggjgjjkgjoiuijouuiyiyyyuyu},
\er{vhfffngghkjgghPPNggjgjjkgjoiuipopiioioouuou} and
\er{hoyuiouigyfghgjh3478zzrrZZjhjhugyuyughggjhjhjyuuygtyffjjjkjljkjjljokjljjkjkopopkklklkjjkjkjjkkokkjoipio;l;l;ll;lkklkjjkkklkkkkkjojjh}
we rewrite the Lagrangian density in \er{vhfffngghkjgghDDnw22mm3488}
or \er{vhfffngghkjgghDDnw22jkjkjkjk22gughhg22mm3488} in the
cartesian coordinate system as:
\begin{multline}\label{vhfffngghkjgghDDmmioiuiujhhjhjjhyryrr}
\frac{1}{8\pi}\left|-\nabla_{\vec
x}\Psi-\frac{1}{c}\frac{\partial\vec A}{\partial t}+\frac{1}{c}\vec
v\times curl_{\vec x}\vec A\right|^2-\frac{1}{8\pi}\left|curl_{\vec
x}\vec A\right|^2-\left(\rho\Psi-\frac{1}{c}\vec A\cdot\vec j\right)
+\mu g\left(\frac{1}{2}\left|\vec u-\vec v\right|^2\right)
\\
-\frac{c^2}{8\pi G}\left|-\nabla_{\vec x}\Phi_0
-\frac{1}{c}\frac{\partial\vec h}{\partial t}
+\frac{1}{c}\vec v\times curl_{\vec x}\vec h-\frac{1}{c}\nabla_{\vec
x}\left(\vec h\cdot\vec v\right)\right|^2+\frac{c^2}{8\pi
G}\left|curl_{\vec x}\vec h\right|^2\\+\frac{c^2\beta}{4\pi
G}\left(\frac{1}{4}\left|d_{\vec x}\vec v+\left\{d_{\vec x}\vec
v\right\}^T\right|^2-\left|\Div_{\vec x}\vec v\right|^2\right)
=L_{ge}=
\\
\frac{1}{4\pi}\left(-\sum_{n=0}^{3}\sum_{k=0}^{3}\sum_{m=0}^{3}\sum_{p=0}^{3}\frac{1}{4}g^{mn}g^{pk}\left(\frac{\partial
A_p}{\partial x^m}-\frac{\partial A_m}{\partial
x^p}\right)\left(\frac{\partial A_k}{\partial x^n}-\frac{\partial
A_n}{\partial x^k}\right)-\sum_{k=0}^{3}4\pi j^k A_k\right)\\
+\frac{\hat \mu}{\sqrt{\left|\text{det}\,G\right|}}
\left(\sum\limits_{m=0}^{3}\sum\limits_{k=0}^{3}g_{mk}\,\hat
u^m_{x^0}\,u^k_{t}\right)\left(\sum\limits_{m=0}^{3}\sum\limits_{k=0}^{3}g_{mk}\,
u^m_{t}\,u^k_{t}\right)^{-1}g\left(\frac{c^2}{2}-\sum\limits_{m=0}^{3}\sum\limits_{k=0}^{3}\frac{c^2}{2}g_{mk}\,
u^m_{t}\,u^k_{t}\right)
\\
+\frac{c^4}{4\pi
G}\left(\sum_{n=0}^{3}\sum_{k=0}^{3}\sum_{m=0}^{3}\sum_{p=0}^{3}\frac{1}{4}g^{mn}g^{pk}\left(\frac{\partial
s_p}{\partial x^m}-\frac{\partial s_m}{\partial
x^p}\right)\left(\frac{\partial s_k}{\partial x^n}-\frac{\partial
s_n}{\partial x^k}\right)\right)\\
-\frac{c^4\beta}{16\pi
G}\left(\left(\sum_{i=0}^{3}\sum_{j=0}^{3}g^{ij}\left(\delta_j
v_i+\delta_i
v_j\right)\right)^2-\sum_{k=0}^{3}\sum_{j=0}^{3}\sum_{m=0}^{3}\sum_{n=0}^{3}\left(\delta_j
v_k+\delta_k v_j\right)g^{km}g^{jn}\left(\delta_m v_n+\delta_n
v_m\right)\right) ,
\end{multline}
where the right hand side is written in the covariant form and the
second equality is valid in every curvilinear coordinate system with
$\frac{\partial t}{\partial x^0}>0$, $(j^0,j^1,j^2,j^3)$ is the
four-vector of the current, $\hat\mu$, $u^j_{t}$ and $u^j_{x^0}$ are
the same as in
\er{btfffjhgjghghijhhjhhhjjkjkkjjhhhjkjjkjhjjkkjklkklijhjhjkmhggmhjhjhkkk},
and $\vec h$ and $\Phi_0$ are given by
\begin{equation}\label{btfffygtgyggyDDnw22jkhk22mm3488jkjljljujhkh}
\vec h=\vec v-\vec
k\quad\text{and}\quad\Phi_0=-\frac{1}{2c}\left|\vec v-\vec
k\right|^2.
\end{equation}
Moreover in the particular case of the relativistic-like choice
$g(s):=-c^2\sqrt{1-\frac{2s}{c^2}}$, by
\er{btfffjhgjghghijhhjhhhjjkjkkjjhhhjkjjkjhjjkkjklkklijhjhjkmhggm}
we can write an alternative to
\er{vhfffngghkjgghDDmmioiuiujhhjhjjhyryrr} as:
\begin{multline}\label{vhfffngghkjgghDDmmioiuiujhhjhjjhyryrrjkjhhh}
\frac{1}{8\pi}\left|-\nabla_{\vec
x}\Psi-\frac{1}{c}\frac{\partial\vec A}{\partial t}+\frac{1}{c}\vec
v\times curl_{\vec x}\vec A\right|^2-\frac{1}{8\pi}\left|curl_{\vec
x}\vec A\right|^2-\left(\rho\Psi-\frac{1}{c}\vec A\cdot\vec j\right)
+\mu g\left(\frac{1}{2}\left|\vec u-\vec v\right|^2\right)
\\
-\frac{c^2}{8\pi G}\left|-\nabla_{\vec x}\Phi_0
-\frac{1}{c}\frac{\partial\vec h}{\partial t}
+\frac{1}{c}\vec v\times curl_{\vec x}\vec h-\frac{1}{c}\nabla_{\vec
x}\left(\vec h\cdot\vec v\right)\right|^2+\frac{c^2}{8\pi
G}\left|curl_{\vec x}\vec h\right|^2\\+\frac{c^2\beta}{4\pi
G}\left(\frac{1}{4}\left|d_{\vec x}\vec v+\left\{d_{\vec x}\vec
v\right\}^T\right|^2-\left|\Div_{\vec x}\vec v\right|^2\right)
=L_{ge}=
\\
\frac{1}{4\pi}\left(-\sum_{n=0}^{3}\sum_{k=0}^{3}\sum_{m=0}^{3}\sum_{p=0}^{3}\frac{1}{4}g^{mn}g^{pk}\left(\frac{\partial
A_p}{\partial x^m}-\frac{\partial A_m}{\partial
x^p}\right)\left(\frac{\partial A_k}{\partial x^n}-\frac{\partial
A_n}{\partial x^k}\right)-\sum_{k=0}^{3}4\pi j^k A_k\right)\\
-\frac{c^2\hat \mu}{\sqrt{\left|\text{det}\,G\right|}}
\sqrt{\left(\sum\limits_{m=0}^{3}\sum\limits_{k=0}^{3}g_{mk}\,\hat
u^m_{x^0}\,\hat u^k_{x^0}\right)}
\\
+\frac{c^4}{4\pi
G}\left(\sum_{n=0}^{3}\sum_{k=0}^{3}\sum_{m=0}^{3}\sum_{p=0}^{3}\frac{1}{4}g^{mn}g^{pk}\left(\frac{\partial
s_p}{\partial x^m}-\frac{\partial s_m}{\partial
x^p}\right)\left(\frac{\partial s_k}{\partial x^n}-\frac{\partial
s_n}{\partial x^k}\right)\right)\\
-\frac{c^4\beta}{16\pi
G}\left(\left(\sum_{i=0}^{3}\sum_{j=0}^{3}g^{ij}\left(\delta_j
v_i+\delta_i
v_j\right)\right)^2-\sum_{k=0}^{3}\sum_{j=0}^{3}\sum_{m=0}^{3}\sum_{n=0}^{3}\left(\delta_j
v_k+\delta_k v_j\right)g^{km}g^{jn}\left(\delta_m v_n+\delta_n
v_m\right)\right) ,
\end{multline}
where the right hand side is written in the covariant form and the
second equality is valid in every curvilinear coordinate system with
$\frac{\partial t}{\partial x^0}>0$. On the other hand, in the case
of fully non-relativistic Lagrangian, where
$g(s)=\left(s-\frac{c^2}{2}\right)$, we can write an alternative to
\er{vhfffngghkjgghDDmmioiuiujhhjhjjhyryrr} as:
\begin{multline}\label{vhfffngghkjgghDDmmioiuiujhhjhjjhyryrrjkjhhhjklkjhhh}
\frac{1}{8\pi}\left|-\nabla_{\vec
x}\Psi-\frac{1}{c}\frac{\partial\vec A}{\partial t}+\frac{1}{c}\vec
v\times curl_{\vec x}\vec A\right|^2-\frac{1}{8\pi}\left|curl_{\vec
x}\vec A\right|^2-\left(\rho\Psi-\frac{1}{c}\vec A\cdot\vec j\right)
+\mu g\left(\frac{1}{2}\left|\vec u-\vec v\right|^2\right)
\\
-\frac{c^2}{8\pi G}\left|-\nabla_{\vec x}\Phi_0
-\frac{1}{c}\frac{\partial\vec h}{\partial t}
+\frac{1}{c}\vec v\times curl_{\vec x}\vec h-\frac{1}{c}\nabla_{\vec
x}\left(\vec h\cdot\vec v\right)\right|^2+\frac{c^2}{8\pi
G}\left|curl_{\vec x}\vec h\right|^2\\+\frac{c^2\beta}{4\pi
G}\left(\frac{1}{4}\left|d_{\vec x}\vec v+\left\{d_{\vec x}\vec
v\right\}^T\right|^2-\left|\Div_{\vec x}\vec v\right|^2\right)
=L_{ge}=
\\
\frac{1}{4\pi}\left(-\sum_{n=0}^{3}\sum_{k=0}^{3}\sum_{m=0}^{3}\sum_{p=0}^{3}\frac{1}{4}g^{mn}g^{pk}\left(\frac{\partial
A_p}{\partial x^m}-\frac{\partial A_m}{\partial
x^p}\right)\left(\frac{\partial A_k}{\partial x^n}-\frac{\partial
A_n}{\partial x^k}\right)-\sum_{k=0}^{3}4\pi j^k A_k\right)\\
-\frac{c^2\hat \mu}{\sqrt{\left|\text{det}\,G\right|}}
\left(\sum\limits_{m=0}^{3}\sum\limits_{k=0}^{3}g_{mk}\,\hat
u^m_{x^0}\,u^k_{t}\right)
\\
+\frac{c^4}{4\pi
G}\left(\sum_{n=0}^{3}\sum_{k=0}^{3}\sum_{m=0}^{3}\sum_{p=0}^{3}\frac{1}{4}g^{mn}g^{pk}\left(\frac{\partial
s_p}{\partial x^m}-\frac{\partial s_m}{\partial
x^p}\right)\left(\frac{\partial s_k}{\partial x^n}-\frac{\partial
s_n}{\partial x^k}\right)\right)\\
-\frac{c^4\beta}{16\pi
G}\left(\left(\sum_{i=0}^{3}\sum_{j=0}^{3}g^{ij}\left(\delta_j
v_i+\delta_i
v_j\right)\right)^2-\sum_{k=0}^{3}\sum_{j=0}^{3}\sum_{m=0}^{3}\sum_{n=0}^{3}\left(\delta_j
v_k+\delta_k v_j\right)g^{km}g^{jn}\left(\delta_m v_n+\delta_n
v_m\right)\right) ,
\end{multline}
where again the right hand side is written in the covariant form and
the second equality is valid in every curvilinear coordinate system
with $\frac{\partial t}{\partial x^0}>0$.

Note that the four-vector of the gravitational potential
$(v^0,v^1,v^2,v^3)$ and the four-covector of the genuine gravity
$(s_0,s_1,s_2,s_3)$ are \underline{not} independent arguments of
$L_{ge}$. Moreover, in the general non-cartesian or curvilinear
coordinate system, knowing $(v^0,v^1,v^2,v^3)$ and thus also knowing
$\{g_{ij}\}_{0\leq i,j\leq 3}$, as before in
\er{huohuioy89gjjhjffffff3478zzrrZZZhjhhjhhjjhhffGGhjjhuiiuihhjkkoioirrkhrrjggggjrrjiokuukuppkpklhjjhj88khhjh},
we can find $(s_0,s_1,s_2,s_3)$ by:
\begin{equation}\label{huohuioy89gjjhjffffff3478zzrrZZZhjhhjhhjjhhffGGhjjhuiiuihhjkkoioirrkhrrjggggjrrjiokuukuppkpklhjjhj88khhjhjkojj}
s_j=\frac{1}{2}\left(\sum_{m=0}^{3} g_{jm}k^m-\sum_{m=0}^{3}
J_{jm}v^m\right)\quad \forall j=0,1,2,3.
\end{equation}
where $(k^0,k^1,k^2,k^3)$ is the four-vector of the potential of
inertia and $\{J_{ij}\}_{0\leq i,j\leq 3}$ is the covariant
kinematic pseudo-metric tensor of inertia, that are prescribed in
every cartesian, non-cartesian or curvilinear coordinate system. On
the other hand, knowing $(s_0,s_1,s_2,s_3)$, we can find
$\{g_{ij}\}_{0\leq i,j\leq 3}$ and $(v^0,v^1,v^2,v^3)$ by the
following covariant identities
\begin{equation}\label{hoyuiouigyfghgjh3478zzrrZZffGGjkkjojj88iihhkhk88jklhkhjkhk8899ouiuui}
{g}_{ij}=J_{ij}+\left(s_i\left(c\,\frac{\partial t}{\partial
x^j}\right)+\left(c\,\frac{\partial t}{\partial
x^i}\right)s_j\right)\quad\quad\forall\,i,j=0,1,2,3
\end{equation}
(see
\er{hoyuiouigyfghgjh3478zzrrZZffGGjkkjojj88iihhkhk88jklhkhjkhk8899}),
and
\begin{equation}\label{hoyuiouigyfghgjh3478zzrrZZffGGjkkjojj88jklkhhk88khjhgh8899kpkjkl}
v^j-k^j=\sum_{m=0}^{3}{\Theta}^{jm}s_m=\sum_{m=0}^{3}\left(k^jk^m-J^{jm}\right)s_m\quad\quad\forall\,j=0,1,2,3
\end{equation}
(see
\er{hoyuiouigyfghgjh3478zzrrZZffGGjkkjojj88jklkhhk88khjhgh8899}).

Once we wrote the Lagrangian density
$L_{ge}:=L_{ge1}\left(v^k,A^k,x^k\right)_{k=0,\ldots
4}=L_{ge2}\left(s^k,A^k,x^k\right)_{k=0,\ldots 4}$ as a covariant
scalar, under the changes of curvilinear coordinate systems, we can
write a covariant Lagrangian as:
\begin{multline}\label{uuuuuuuujlkjjklkjk}
J_{ge1}(v^k,A^k)=J_{ge1}(s^k,A^k):=\int_{(x^0,x^1,x^2,x^3)}L_{ge1}\left(v^k,A^k,x^k\right)\,\sqrt{\left|\text{det}\,G\right|}\,dx^0dx^1dx^2dx^3
\\=\int_{(x^0,x^1,x^2,x^3)}L_{ge2}\left(s^k,A^k,x^k\right)\,\sqrt{\left|\text{det}\,G\right|}\,dx^0dx^1dx^2dx^3.
\end{multline}
Although we need a term $\sqrt{\left|\text{det}\,G\right|}$ for the
covariance of the Lagrangian in curvilinear coordinate systems, in
the cartesian coordinate systems we always have
$\sqrt{\left|\text{det}\,G\right|}=1$.

Finally, note again, that the four-vector of the gravitational
potential, as an independent argument, is restricted by
\er{fgjfjhfgfgfgfhihhjhhjjhhjhjjhjhhjhjkjh}:
\begin{equation}\label{fgjfjhfgfgfgfhihhjhhjjhhjhjjhjhhjhjkjhojjjjjlljj}
\sum_{k=0}^{3}c\frac{\partial t}{\partial x^k}v^k
=1.
\end{equation}
So the four-vector $(v^0,v^1,v^2,v^3)$ actually contains three
independent scalar functions similarly as in cartesian coordinate
systems where we have $v^0=1$. On the other hand if we consider
$(s_0,s_1,s_2,s_3)$ instead of $(v^0,v^1,v^2,v^3)$ , as an
independent argument, then by
\er{fgjfjhgghhgjghjhjkkkkgjghghuiiiuujhjhljhgugii88} and
\er{hoyuiouigyfghgjh3478zzrrZZffGGjkkjojj88iihhkhk88jklhkhjkhk8899ouiuui}
it is restricted by the following identity
\begin{equation}\label{fgjfjhgghhgjghjhjkkkkgjghghuiiiuujhjhljhgugii88jklkljjkk;kkl}
\text{det}\,\left(\left\{J_{ij}+s_i\left(c\,\frac{\partial
t}{\partial x^j}\right)+\left(c\,\frac{\partial t}{\partial
x^i}\right)s_j\right\}_{0\leq i,j\leq
3}\right)=\text{det}\,\left(\{J_{ij}\}_{0\leq i,j\leq 3}\right),
\end{equation}
which is valid in every cartesian, non-cartesian or curvilinear
coordinate system. So the four-covector $(s_0,s_1,s_2,s_3)$ also
contains three independent scalar functions.


\section{Relativistic-like Dirac equation}\label{fgffggfggfhhjh}
\subsection{Classical Relativistik-like Lagrangian and Hamiltonian
of the motion of particles}
As in
\er{btfffygtgyggyijhhkkSYSPNuhuygyygyggyuyyuyuykuhghgjjojiyyyuyuuyy}
consider the relativistic-like Lagrangian of the motion of the
particle with mass $m$ and charge $\sigma$ in the outer
gravitational and electromagnetic fields and additional field with
potential $V(\vec x,t)$:
\begin{multline}\label{btfffygtgyggyijhhkkSYSPNuhuygyygyggyuyyuyuykuhghgjjojiyyyuyuuyyKK}
J_{rl}(\vec r)=\int_0^T L_0\left(\frac{d\vec r}{dt},\vec
r,t\right)dt:=\\ \int_0^T\left\{
-mc^2\sqrt{1-\frac{1}{c^2}\left|\frac{d\vec r}{dt}-\vec v(\vec
r,t)\right|^2}-\sigma\left(\Psi(\vec r,t)-\frac{1}{c}\vec A(\vec
r,t)\cdot\frac{d\vec r}{dt}\right)+V\left(\vec r,t\right)\right\}dt.
\end{multline}
Next define the generalized momentum of the particle by
\begin{equation}\label{guytyurtydftyiujhSYShmhmKK}
\vec P:=\nabla_{\vec r'}L_0\left(\vec r',\vec
r,t\right)=m\left(1-\frac{1}{c^2}\left|\frac{d\vec r}{dt}-\vec
v(\vec r,t)\right|^2\right)^{-\frac{1}{2}}\left(\frac{d\vec
r}{dt}-\vec v(\vec r,t)\right)+\frac{\sigma}{c}\vec A(\vec r,t).
\end{equation}
Then
\begin{equation}\label{guytyurtydftyiujhioyhuyhSYShmhmKK}
\left(1-\frac{1}{c^2}\left|\frac{d\vec r}{dt}-\vec v(\vec
r,t)\right|^2\right)^{-\frac{1}{2}}\left(\frac{d\vec r}{dt}-\vec
v(\vec r,t)\right)=\left(\frac{1}{m}\vec P-\frac{\sigma}{mc}\vec
A(\vec r,t)\right).
\end{equation}
So
\begin{equation}\label{guytyurtydftyiujhioyhuyhSYShmhmuiiuiuKK}
\frac{d\vec r}{dt}=\left(1+\frac{1}{c^2}\left|\frac{1}{m}\vec
P-\frac{\sigma}{mc}\vec A(\vec
r,t)\right|^2\right)^{-\frac{1}{2}}\left(\frac{1}{m}\vec
P-\frac{\sigma}{mc}\vec A(\vec r,t)\right)+\vec v(\vec r,t).
\end{equation}
Thus if we consider a Hamiltonian
\begin{equation}\label{vhfffngghkjgghfjjhyjjfgSYShmhmKK}
H_0\left(\vec P,\vec r,t\right):=\vec P\cdot\frac{d\vec
r}{dt}-L_0\left(\frac{d\vec r}{dt},\vec r,t\right)
\end{equation}
then by \er{vhfffngghkjgghfjjhyjjfgSYShmhmKK},
\er{btfffygtgyggyijhhkkSYSPNuhuygyygyggyuyyuyuykuhghgjjojiyyyuyuuyyKK},
\er{guytyurtydftyiujhioyhuyhSYShmhmKK} and
\er{guytyurtydftyiujhioyhuyhSYShmhmuiiuiuKK} we have:
\begin{multline}\label{vhfffngghkjgghfjjghghghSYShmyuuiiuuhmhmKK}
H_0\left(\vec P,\vec r,t\right)=-V\left(\vec r,t\right)+\vec
P\cdot\frac{d\vec r}{dt}
+\left(mc^2\left(1-\frac{1}{c^2}\left|\frac{d\vec r}{dt}-\vec v(\vec
r,t)\right|^2\right)^{\frac{1}{2}}+\sigma\left(\Psi(\vec
r,t)-\frac{1}{c}\vec A(\vec
r,t)\cdot\frac{d\vec r}{dt}\right)\right)=\\
-V\left(\vec
r,t\right)+mc^2\left(1+\frac{1}{c^2}\left|\frac{1}{m}\vec
P-\frac{\sigma}{mc}\vec A(\vec r,t)\right|^2\right)^{-\frac{1}{2}}
\\
+\vec P\cdot\left(\left(1+\frac{1}{c^2}\left|\frac{1}{m}\vec
P-\frac{\sigma}{mc}\vec A(\vec
r,t)\right|^2\right)^{-\frac{1}{2}}\left(\frac{1}{m}\vec
P-\frac{\sigma}{mc}\vec A(\vec r,t)\right)+\vec v(\vec
r,t)\right)\\+\sigma\left(\Psi(\vec r,t)-\frac{1}{c}\vec A(\vec
r,t)\cdot\left(\left(1+\frac{1}{c^2}\left|\frac{1}{m}\vec
P-\frac{\sigma}{mc}\vec A(\vec
r,t)\right|^2\right)^{-\frac{1}{2}}\left(\frac{1}{m}\vec
P-\frac{\sigma}{mc}\vec A(\vec r,t)\right)+\vec v(\vec
r,t)\right)\right)\\=
mc^2\left(1+\frac{1}{c^2}\left|\frac{1}{m}\vec
P-\frac{\sigma}{mc}\vec A(\vec r,t)\right|^2\right)^{\frac{1}{2}}
+\sigma\left(\Psi(\vec r,t)-\frac{1}{c}\vec v(\vec r,t)\cdot\vec
A(\vec r,t)\right)-V\left(\vec r,t\right)+\vec v(\vec r,t)\cdot\vec
P.
\end{multline}
So, the relativistic-like Hamiltonian for a macro-particles has the
form:
\begin{multline}\label{vhfffngghkjgghfjjghghghSYShmyuuiiuuhmhmiopoopKK}
H_0\left(\vec P,\vec r,t\right)=
mc^2\left(1+\frac{1}{c^2}\left|\frac{1}{m}\vec
P-\frac{\sigma}{mc}\vec A(\vec r,t)\right|^2\right)^{\frac{1}{2}}
+\sigma\left(\Psi(\vec r,t)-\frac{1}{c}\vec v(\vec r,t)\cdot\vec
A(\vec r,t)\right)\\-V\left(\vec r,t\right)+\vec v(\vec
r,t)\cdot\vec P.
\end{multline}
In particular, if similarly to
\er{fgjfjhgghhgjghjhjijhojihjhjjijhjjjljljpkhkhjgjg} we define the
four-dimensional generalized momentum $(P_0,P_1,P_2,P_3)$ as:
\begin{equation}\label{fgjfjhgghhgjghjhjijhojihjhjjijhjjjljljpkhkhjgjgjj}
(P_0,P_1,P_2,P_3):=\left(\frac{1}{c}H_0,-\vec
P\right)\quad\text{where}\quad
P_0=\frac{1}{c}H_0\;\;\text{and}\;\;(P_1,P_2,P_3)=-\vec P,
\end{equation}
Then, since by \er{guytyurtydftyiujhSYShmhmKK} and
\er{vhfffngghkjgghfjjghghghSYShmyuuiiuuhmhmiopoopKK}, under the
change of non-inertial cartesian coordinate system $H_0$ and $\vec
P$ transform as
\begin{equation}\label{vhfffngghhjghhgPPNghghghutghfflklhjkjhjhjjgjkghhjhhjhjkhjghlkklljj}
\begin{cases}
H_0'=
H_0+\left(\frac{dA}{dt}(t)\cdot\vec x+\frac{d\vec
z}{dt}(t)\right)\cdot\left(A(t)\cdot\vec P\right)
\\
\vec P'=A(t)\cdot \vec P,
\end{cases}
\end{equation}
by comparing
\er{vhfffngghhjghhgPPNghghghutghfflklhjkjhjhjjgjkghhjhhjhjkhjghlkklljj}
with \er{fgjfjhgghhgjghjhjijhojpiiihjhjuiouj} we deduce that the
four-dimensional momentum $(P_0,P_1,P_2,P_3)$ is a
four-\underline{covector} on the group $\mathcal{S}_0$ that is the
group of changes of cartesian non-inertial coordinate systems.

\subsection{Dirac equation}
Next consider the motion of a spin-half quantum relativistic-like
micro-particle with inertial mass $m$ and charge $\sigma$ in the
outer gravitational and electromagnetical field with characteristics
$\vec v(\vec x,t)$, $\vec A(\vec x,t)$ and $\Psi(\vec x,t)$ and
additional conservative field with potential $V(\vec y,t)$. The
evolution equation for this particle is
\begin{equation}\label{vhfffngghkjgghfjjghghghhjghjgghkghggkghghjghSYSPNnnkkJJ}
i\hbar\frac{\partial\psi}{\partial t}=\hat H_0\cdot\psi,
\end{equation}
where $\psi(\vec x,t)=\left(\psi_1(\vec x,t),\psi_2(\vec
x,t)\right)\in\mathbb{C}^2\times \mathbb{C}^2$ is a four-component
wave function and $\hat H_0$ is the Hamiltonian operator. Since the
relativistic-like Hamiltonian for a macro-particles has the form
\er{vhfffngghkjgghfjjghghghSYShmyuuiiuuhmhmiopoopKK},
%
%
%
analogously to the usual Dirac Hamiltonian operator, we built the
Hermitian Hamiltonian operator as $\hat H_0\cdot\psi=\left(\hat
H_1\cdot\psi,\hat H_2\cdot\psi\right)$, where
\begin{multline}\label{vhfffngghkjgghfjjghghghSYShmyuuiiuuhmhmiopoopnniukjhjkkkJJ}
\hat H_1\cdot\psi= mc^2\psi_1-c\,\vec S\cdot\left(i\hbar\nabla_{\vec
x}\psi_2+\frac{\sigma}{c}\vec A(\vec x,t)\psi_2\right)
%
%
%
+\sigma\left(\Psi(\vec x,t)-\frac{1}{c}\vec v(\vec x,t)\cdot\vec
A(\vec x_j,t)\right)\psi_1\\-V\left(\vec
x,t\right)\psi_1-\frac{i\hbar}{2}div_{\vec x}\left\{\psi_1\vec
v(\vec x,t)\right\}-\frac{i\hbar}{2}\nabla_{\vec x}\psi_1\cdot\vec
v(\vec x,t)+\frac{\hbar}{4}\vec S\cdot\left(curl_{\vec x}\vec v(\vec
x,t)\,\psi_1\right)\\-\frac{(g-1)\sigma\hbar}{2mc}\vec
S\cdot\left(curl_{\vec x}\vec A(\vec x,t)\,\psi_1\right),
\end{multline}
and
\begin{multline}\label{vhfffngghkjgghfjjghghghSYShmyuuiiuuhmhmiopoopnniukjhjkhhjkkJJ}
\hat H_2\cdot\psi= -mc^2\psi_2-c\,\vec
S\cdot\left(i\hbar\nabla_{\vec x}\psi_1+\frac{\sigma}{c}\vec A(\vec
x,t)\psi_1\right)
%
%
%
+\sigma\left(\Psi(\vec x,t)-\frac{1}{c}\vec v(\vec x,t)\cdot\vec
A(\vec x,t)\right)\psi_2\\-V\left(\vec
x,t\right)\psi_2-\frac{i\hbar}{2}div_{\vec x}\left\{\psi_2\vec
v(\vec x,t)\right\}-\frac{i\hbar}{2}\nabla_{\vec x}\psi_2\cdot\vec
v(\vec x,t)+\frac{\hbar}{4}\vec S\cdot\left(curl_{\vec x}\vec v(\vec
x,t)\,\psi_2\right)\\-\frac{(g-1)\sigma\hbar}{2mc}\vec
S\cdot\left(curl_{\vec x}\vec A(\vec x,t)\,\psi_2\right),
\end{multline}
where $\vec S:=(S_1,S_2,S_3)$ and
$$S_1=\left(\begin{matrix}0&1\\1&0\end{matrix}\right),\quad
S_2=\left(\begin{matrix}0&-i\\i&0\end{matrix}\right)\quad
S_3=\left(\begin{matrix}1&0\\0&-1\end{matrix}\right)$$ are Pauli
matrices and $g$ is a constant that depend on the type of the
particle (for electron we have $g=1$). As before for the
Schr\"{o}dinger-Pauli equation, we added an additional term to the
Hamiltonian, namely $\frac{\hbar}{4}\vec S\cdot\left(curl_{\vec
x}\vec v(\vec x,t)\,\psi\right)$. Although in the case of the
Newtonian-type gravity, this term vanishes in inertial coordinate
systems, it provides however, invariance of our Dirac-type equation,
under the change of non-inertial coordinate systems as we will see
below. Thus, we have the following two evolution equations that we
call together Dirac system of equations:
\begin{multline}\label{vhfffngghkjgghfjjghghghSYShmyuuiiuuhmhmiopoopnniukjhjkk;l;lkkkJJ}
i\hbar\frac{\partial\psi_1}{\partial t}=mc^2\psi_1-c\,\vec
S\cdot\left(i\hbar\nabla_{\vec x}\psi_2+\frac{\sigma}{c}\vec A(\vec
x,t)\psi_2\right)
%
%
%
+\sigma\left(\Psi(\vec x,t)-\frac{1}{c}\vec v(\vec x,t)\cdot\vec
A(\vec x_j,t)\right)\psi_1\\-V\left(\vec
x,t\right)\psi_1-\frac{i\hbar}{2}div_{\vec x}\left\{\psi_1\vec
v(\vec x,t)\right\}-\frac{i\hbar}{2}\nabla_{\vec x}\psi_1\cdot\vec
v(\vec x,t)+\frac{\hbar}{4}\vec S\cdot\left(curl_{\vec x}\vec v(\vec
x,t)\,\psi_1\right)\\-\frac{(g-1)\sigma\hbar}{2mc}\vec
S\cdot\left(curl_{\vec x}\vec A(\vec x,t)\,\psi_1\right),
\end{multline}
and
\begin{multline}\label{vhfffngghkjgghfjjghghghSYShmyuuiiuuhmhmiopoopnniukjhjkhhjl;kkJJ}
i\hbar\frac{\partial\psi_2}{\partial t}=-mc^2\psi_2-c\,\vec
S\cdot\left(i\hbar\nabla_{\vec x}\psi_1+\frac{\sigma}{c}\vec A(\vec
x,t)\psi_1\right)
%
%
%
+\sigma\left(\Psi(\vec x,t)-\frac{1}{c}\vec v(\vec x,t)\cdot\vec
A(\vec x,t)\right)\psi_2\\-V\left(\vec
x,t\right)\psi_2-\frac{i\hbar}{2}div_{\vec x}\left\{\psi_2\vec
v(\vec x,t)\right\}-\frac{i\hbar}{2}\nabla_{\vec x}\psi_2\cdot\vec
v(\vec x,t)+\frac{\hbar}{4}\vec S\cdot\left(curl_{\vec x}\vec v(\vec
x,t)\,\psi_2\right)\\-\frac{(g-1)\sigma\hbar}{2mc}\vec
S\cdot\left(curl_{\vec x}\vec A(\vec x,t)\,\psi_2\right).
\end{multline}
Then we can rewrite Dirac equations as:
\begin{multline}\label{vhfffngghkjgghfjjghghghSYShmyuuiiuuhmhmiopoopnniukjhjkk;l;lkkkJJjjkjk}
i\hbar\left(\frac{\partial\psi_1}{\partial t}+\frac{1}{2}div_{\vec
x}\left\{\psi_1\vec v(\vec x,t)\right\}+\frac{1}{2}\nabla_{\vec
x}\psi_1\cdot\vec v(\vec x,t)\right)=mc^2\psi_1-c\,\vec
S\cdot\left(i\hbar\nabla_{\vec x}\psi_2+\frac{\sigma}{c}\vec A(\vec
x,t)\psi_2\right)
%
%
%
\\+\sigma\left(\Psi(\vec x,t)-\frac{1}{c}\vec v(\vec x,t)\cdot\vec
A(\vec x,t)\right)\psi_1-V\left(\vec
x,t\right)\psi_1+\frac{\hbar}{4}\vec S\cdot\left(curl_{\vec x}\vec
v(\vec x,t)\,\psi_1\right)\\
-\frac{(g-1)\sigma\hbar}{2mc}\vec S\cdot\left(curl_{\vec x}\vec
A(\vec x,t)\,\psi_1\right),
\end{multline}
and
\begin{multline}\label{vhfffngghkjgghfjjghghghSYShmyuuiiuuhmhmiopoopnniukjhjkhhjl;kkJJkllkkl}
i\hbar\left(\frac{\partial\psi_2}{\partial t}+\frac{1}{2}div_{\vec
x}\left\{\psi_2\vec v(\vec x,t)\right\}+\frac{1}{2}\nabla_{\vec
x}\psi_2\cdot\vec v(\vec x,t)\right)=-mc^2\psi_2-c\,\vec
S\cdot\left(i\hbar\nabla_{\vec x}\psi_1+\frac{\sigma}{c}\vec A(\vec
x,t)\psi_1\right)
%
%
%
\\+\sigma\left(\Psi(\vec x,t)-\frac{1}{c}\vec v(\vec x,t)\cdot\vec
A(\vec x,t)\right)\psi_2-V\left(\vec
x,t\right)\psi_2+\frac{\hbar}{4}\vec S\cdot\left(curl_{\vec x}\vec
v(\vec x,t)\,\psi_2\right)\\
-\frac{(g-1)\sigma\hbar}{2mc}\vec S\cdot\left(curl_{\vec x}\vec
A(\vec x,t)\,\psi_2\right).
\end{multline}
Then similarly to the proof of Theorem \ref{gjghghgghgintintrrZZ}
about the invariance of Shr\"{o}dinger-Pauli equation we can prove
the following Theorem for Dirac equations:

\begin{theorem}\label{gjghghgghgintintrrZZjkh}
Consider that the change of some cartesian coordinate system $(*)$
to another cartesian coordinate system $(**)$ is given by
\begin{equation}\label{noninchgravortbstrjgghguittu2intrrrZZjkjjkjkjk}
\begin{cases}
\vec x'=A(t)\cdot\vec x+\vec z(t),\\
t'=t,
\end{cases}
\end{equation}
where $A(t)\in SO(3)$ is a rotation. Next, assume that in the
coordinate system $(**)$ we observe a validity of the Dirac
equations of the form:
\begin{multline}\label{vhfffngghkjgghfjjghghghSYShmyuuiiuuhmhmiopoopnniukjhjkk;l;lkkkJJjjkjk;lk;kl}
i\hbar\left(\frac{\partial\psi'_1}{\partial t'}+\frac{1}{2}div_{\vec
x'}\left\{\psi'_1\vec v'(\vec x',t')\right\}+\frac{1}{2}\nabla_{\vec
x'}\psi'_1\cdot\vec v'(\vec x',t')\right)=m'c^2\psi'_1-c\,\vec
S\cdot\left(i\hbar\nabla_{\vec x'}\psi'_2+\frac{\sigma'}{c}\vec
A'(\vec x',t)\psi'_2\right)
%
%
%
\\+\sigma'\left(\Psi'(\vec x',t')-\frac{1}{c}\vec v'(\vec x',t')\cdot\vec
A'(\vec x',t')\right)\psi'_1-V'\left(\vec
x',t'\right)\psi'_1+\frac{\hbar}{4}\vec S\cdot\left(curl_{\vec
x'}\vec v'(\vec x',t')\,\psi'_1\right)\\
-\frac{(g'-1)\sigma'\hbar}{2m'c}\vec S\cdot\left(curl_{\vec x'}\vec
A'(\vec x',t)\,\psi'_1\right),
\end{multline}
and
\begin{multline}\label{vhfffngghkjgghfjjghghghSYShmyuuiiuuhmhmiopoopnniukjhjkhhjl;kkJJkllkklkk}
i\hbar\left(\frac{\partial\psi'_2}{\partial t'}+\frac{1}{2}div_{\vec
x'}\left\{\psi'_2\vec v'(\vec x',t')\right\}+\frac{1}{2}\nabla_{\vec
x'}\psi'_2\cdot\vec v'(\vec x',t')\right)=\\-m'c^2\psi'_2-c\,\vec
S\cdot\left(i\hbar\nabla_{\vec x'}\psi'_1+\frac{\sigma'}{c}\vec
A'(\vec x',t')\psi'_1\right)
%
%
%
+\sigma'\left(\Psi'(\vec x',t')-\frac{1}{c}\vec v'(\vec
x',t')\cdot\vec A'(\vec x',t')\right)\psi'_2\\-V'\left(\vec
x',t'\right)\psi'_2+\frac{\hbar}{4}\vec S\cdot\left(curl_{\vec
x'}\vec v'(\vec x',t')\,\psi'_2\right)\\
-\frac{(g'-1)\sigma'\hbar}{2m'c}\vec S\cdot\left(curl_{\vec x'}\vec
A'(\vec x',t)\,\psi'_2\right),
\end{multline}
where $\psi=(\psi_1,\psi_2)\in\mathbb{C}^2\times \mathbb{C}^2$ is a
four-component wave function. Then in the coordinate system $(*)$ we
have the validity of Dirac equations of the same as
\er{vhfffngghkjgghfjjghghghSYShmyuuiiuuhmhmiopoopnniukjhjkk;l;lkkkJJjjkjk;lk;kl}
and
\er{vhfffngghkjgghfjjghghghSYShmyuuiiuuhmhmiopoopnniukjhjkhhjl;kkJJkllkklkk}
form:
\begin{multline}\label{vhfffngghkjgghfjjghghghSYShmyuuiiuuhmhmiopoopnniukjhjkk;l;lkkkJJjjkjkjkjk}
i\hbar\left(\frac{\partial\psi_1}{\partial t}+\frac{1}{2}div_{\vec
x}\left\{\psi_1\vec v(\vec x,t)\right\}+\frac{1}{2}\nabla_{\vec
x}\psi_1\cdot\vec v(\vec x,t)\right)=mc^2\psi_1-c\,\vec
S\cdot\left(i\hbar\nabla_{\vec x}\psi_2+\frac{\sigma}{c}\vec A(\vec
x,t)\psi_2\right)
%
%
%
\\+\sigma\left(\Psi(\vec x,t)-\frac{1}{c}\vec v(\vec x,t)\cdot\vec
A(\vec x,t)\right)\psi_1-V\left(\vec
x,t\right)\psi_1+\frac{\hbar}{4}\vec S\cdot\left(curl_{\vec x}\vec
v(\vec x,t)\,\psi_1\right)\\-\frac{(g-1)\sigma\hbar}{2mc}\vec
S\cdot\left(curl_{\vec x}\vec A(\vec x,t)\,\psi_1\right),
\end{multline}
and
\begin{multline}\label{vhfffngghkjgghfjjghghghSYShmyuuiiuuhmhmiopoopnniukjhjkhhjl;kkJJkllkkljjj}
i\hbar\left(\frac{\partial\psi_2}{\partial t}+\frac{1}{2}div_{\vec
x}\left\{\psi_2\vec v(\vec x,t)\right\}+\frac{1}{2}\nabla_{\vec
x}\psi_2\cdot\vec v(\vec x,t)\right)=-mc^2\psi_2-c\,\vec
S\cdot\left(i\hbar\nabla_{\vec x}\psi_1+\frac{\sigma}{c}\vec A(\vec
x,t)\psi_1\right)
%
%
%
\\+\sigma\left(\Psi(\vec x,t)-\frac{1}{c}\vec v(\vec x,t)\cdot\vec
A(\vec x,t)\right)\psi_2-V\left(\vec
x,t\right)\psi_2+\frac{\hbar}{4}\vec S\cdot\left(curl_{\vec x}\vec
v(\vec x,t)\,\psi_2\right)\\
-\frac{(g-1)\sigma\hbar}{2mc}\vec S\cdot\left(curl_{\vec x}\vec
A(\vec x,t)\,\psi_2\right),
\end{multline}
provided that
\begin{equation}\label{yuythfgfyftydtydtydtyddyyyhhddhhhredPPN111hgghjgintintrrZZkjhh}
\begin{cases}
V'=V,\\
\sigma'=\sigma,\\
m'=m,\\
\vec v'=A(t)\cdot \vec v+\frac{dA}{dt}(t)\cdot\vec x+\frac{d\vec z}{dt}(t),\\
\vec A'=A(t)\cdot \vec A,\\
\Psi'-\vec v'\cdot\vec A'=\Psi-\vec v\cdot\vec A,\\
\psi'_1=U(t)\cdot\psi_1,\\
\psi'_2=U(t)\cdot\psi_2,
\end{cases}
\end{equation}
where, as before, $U(t)\in SU(2)$ is characterized by:
\begin{equation}\label{gyfyfgfgfgghZZjkjk}
U^*(t)\cdot\vec S\cdot U(t)=A(t)\cdot\vec S,
\end{equation}
that means
\begin{multline*}
\left(U^*(t)\cdot S_1\cdot U(t),U^*(t)\cdot S_2\cdot
U(t),U^*(t)\cdot S_3\cdot U(t)\right)=\\
\left(a_{11}(t)S_1+a_{12}(t)S_2+a_{13}(t)S_3\,,\,a_{21}(t)S_1+a_{22}(t)S_2+a_{23}(t)S_3\,,\,a_{31}(t)S_1+a_{32}(t)S_2+a_{33}(t)S_3\right),
\end{multline*}
where $A(t)=\left\{a_{mk}(t)\right\}_{\{1\leq m,k\leq 3\}}$.
\end{theorem}

Next, in the case that our particle has a positive energy, define
$$(\phi_1,\phi_2)=\left(e^{-\frac{ic^2mt}{\hbar}}\psi_1,e^{-\frac{ic^2mt}{\hbar}}\psi_2\right).$$
Then we rewrite the Dirac equations
\er{vhfffngghkjgghfjjghghghSYShmyuuiiuuhmhmiopoopnniukjhjkk;l;lkkkJJ}
and
\er{vhfffngghkjgghfjjghghghSYShmyuuiiuuhmhmiopoopnniukjhjkhhjl;kkJJ}
as
\begin{multline}\label{vhfffngghkjgghfjjghghghSYShmyuuiiuuhmhmiopoopnniukjhjkk;l;lkhjjkkkJJ}
i\hbar\frac{\partial\phi_1}{\partial t}= -c\,\vec
S\cdot\left(i\hbar\nabla_{\vec x}\phi_2+\frac{\sigma}{c}\vec A(\vec
x,t)\phi_2\right)
%
%
%
+\sigma\left(\Psi(\vec x,t)-\frac{1}{c}\vec v(\vec x,t)\cdot\vec
A(\vec x,t)\right)\phi_1\\-V\left(\vec
x,t\right)\phi_1-\frac{i\hbar}{2}div_{\vec x}\left\{\phi_1\vec
v(\vec x,t)\right\}-\frac{i\hbar}{2}\nabla_{\vec x}\phi_1\cdot\vec
v(\vec x,t)+\frac{\hbar}{4}\vec S\cdot\left(curl_{\vec x}\vec v(\vec
x,t)\,\phi_1\right)\\-\frac{(g-1)\sigma\hbar}{2mc}\vec
S\cdot\left(curl_{\vec x}\vec A(\vec x,t)\,\phi_1\right),
\end{multline}
and
\begin{multline}\label{vhfffngghkjgghfjjghghghSYShmyuuiiuuhmhmiopoopnniukjhjkhhjl;jkljhkkJJ}
i\hbar\frac{\partial\phi_2}{\partial t}= -2mc^2\phi_2-c\,\vec
S\cdot\left(i\hbar\nabla_{\vec x}\phi_1+\frac{\sigma}{c}\vec A(\vec
x,t)\phi_1\right)
%
%
%
+\sigma\left(\Psi(\vec x,t)-\frac{1}{c}\vec v(\vec x,t)\cdot\vec
A(\vec x,t)\right)\phi_2\\-V\left(\vec
x,t\right)\phi_2-\frac{i\hbar}{2}div_{\vec x}\left\{\phi_2\vec
v(\vec x,t)\right\}-\frac{i\hbar}{2}\nabla_{\vec x}\phi_2\cdot\vec
v(\vec x,t)+\frac{\hbar}{4}\vec S\cdot\left(curl_{\vec x}\vec v(\vec
x,t)\,\phi_2\right)\\-\frac{(g-1)\sigma\hbar}{2mc}\vec
S\cdot\left(curl_{\vec x}\vec A(\vec x,t)\,\phi_2\right).
\end{multline}
Thus from
\er{vhfffngghkjgghfjjghghghSYShmyuuiiuuhmhmiopoopnniukjhjkhhjl;jkljhkkJJ},
in the non-relativistic limit we have,
\begin{equation}\label{vhfffngghkjgghfjjghghghSYShmyuuiiuuhmhmiopoopnniukjhjkhhjl;jkljhjjkllkjkkJJ}
2mc^2\phi_2\approx-c\,\vec S\cdot\left(i\hbar\nabla_{\vec
x}\phi_1+\frac{\sigma}{c}\vec A(\vec x,t)\phi_1\right).
%
%
%
\end{equation}
I.e.
\begin{equation}\label{vhfffngghkjgghfjjghghghSYShmyuuiiuuhmhmiopoopnniukjhjkhhjl;jkljhjjkllkjhykkJJ}
\phi_2\approx-\frac{1}{2c m}\vec S\cdot\left(i\hbar\nabla_{\vec
x}\phi_1+\frac{\sigma}{c}\vec A(\vec x,t)\phi_1\right).
%
%
%
\end{equation}
Thus inserting
\er{vhfffngghkjgghfjjghghghSYShmyuuiiuuhmhmiopoopnniukjhjkhhjl;jkljhjjkllkjhykkJJ}
into
\er{vhfffngghkjgghfjjghghghSYShmyuuiiuuhmhmiopoopnniukjhjkk;l;lkhjjkkkJJ}
gives
\begin{multline}\label{vhfffngghkjgghfjjghghghSYShmyuuiiuuhmhmiopoopnniukjhjkk;l;lkhjjkihjjhkkJJ}
i\hbar\frac{\partial\phi_1}{\partial t}\approx
\\
\frac{1}{2m}\vec S\cdot\left(i\hbar\nabla_{\vec x}\left(\vec
S\cdot\left(i\hbar\nabla_{\vec x}\phi_1+\frac{\sigma}{c}\vec A(\vec
x,t)\phi_1\right)\right)\right)+\frac{1}{2m}\vec
S\cdot\left(\frac{\sigma}{c}\vec A(\vec x,t)\left(\vec
S\cdot\left(i\hbar\nabla_{\vec x}\phi_1+\frac{\sigma}{c}\vec A(\vec
x,t)\phi_1\right)\right)\right)
%
%
%
\\+\sigma\left(\Psi(\vec x,t)-\frac{1}{c}\vec v(\vec x,t)\cdot\vec
A(\vec x,t)\right)\phi_1-V\left(\vec
x,t\right)\phi_1\\-\frac{i\hbar}{2}div_{\vec x}\left\{\phi_1\vec
v(\vec x,t)\right\}-\frac{i\hbar}{2}\nabla_{\vec x}\phi_1\cdot\vec
v(\vec x,t)+\frac{\hbar}{4}\vec S\cdot\left(curl_{\vec x}\vec v(\vec
x,t)\,\phi_1\right)-\frac{(g-1)\sigma\hbar}{2mc}\vec
S\cdot\left(curl_{\vec x}\vec A(\vec x,t)\,\phi_1\right),
\end{multline}
that we rewrite as a non-relativistic Shr\"{o}dinger-Pauli equation,
that we studied above:
\begin{multline}\label{vhfffngghkjgghfjjghghghSYShmyuuiiuuhmhmiopoopnniukjhjkk;l;lkhjjkihjjhkkJJllkkk}
i\hbar\frac{\partial\phi_1}{\partial
t}\approx-\frac{\hbar^2}{2m}\Delta_{\vec
x}\phi_1+\frac{i\hbar\sigma}{2mc}div_{\vec x}\left\{\phi_1\vec
A(\vec x,t)\right\}+\frac{i\hbar\sigma}{2mc}\nabla_{\vec
x}\phi_1\cdot\vec A(\vec x,t)+\frac{\sigma^2}{2mc^2}\left|\vec
A(\vec x,t)\right|^2\phi_1\\-\frac{g\sigma\hbar}{2mc}\vec
S\cdot\left(curl_{\vec x}\vec A(\vec x,t)\,\phi_1\right)
%
%
%
+\sigma\left(\Psi(\vec x,t)-\frac{1}{c}\vec v(\vec x,t)\cdot\vec
A(\vec x,t)\right)\phi_1-V\left(\vec
x,t\right)\phi_1\\-\frac{i\hbar}{2}div_{\vec x}\left\{\phi_1\vec
v(\vec x,t)\right\}-\frac{i\hbar}{2}\nabla_{\vec x}\phi_1\cdot\vec
v(\vec x,t)+\frac{\hbar}{4}\vec S\cdot\left(curl_{\vec x}\vec v(\vec
x,t)\,\phi_1\right)
.
\end{multline}

Next, consider a Lagrangian density $L$ associated the motion of a
spin-half quantum relativistic-like micro-particle with inertial
mass $m$ and charge $\sigma$ in the outer gravitational and
electromagnetical field with characteristics $\vec v(\vec x,t)$,
$\vec A(\vec x,t)$ and $\Psi(\vec x,t)$ and additional conservative
field with potential $V(\vec y,t)$:
\begin{multline}\label{vhfffngghkjgghDDmmkkkZZkkk}
L\left(\psi,\vec x,t\right):=
\frac{i\hbar}{2}\left(\left(\frac{\partial\psi_1}{\partial t}+\vec
v\cdot\nabla_{\vec
x}\psi_1\right)\cdot\bar\psi_1-\psi_1\cdot\left(\frac{\partial\bar\psi_1}{\partial
t}+\vec v\cdot\nabla_{\vec
x}\bar\psi_1\right)\right)\\+\frac{c}{2}\left(\left(\vec
S\cdot\left(i\hbar\nabla_{\vec x}\psi_2+\frac{\sigma}{c}\vec A(\vec
x,t)\psi_2\right)\right)\cdot\bar\psi_1-\psi_1\cdot\left(\vec
S\cdot\left(i\hbar\nabla_{\vec x}\bar\psi_2-\frac{\sigma}{c}\vec
A(\vec x,t)\bar\psi_2\right)\right)\right)\\+
\frac{i\hbar}{2}\left(\left(\frac{\partial\psi_2}{\partial t}+\vec
v\cdot\nabla_{\vec
x}\psi_2\right)\cdot\bar\psi_2-\psi_2\cdot\left(\frac{\partial\bar\psi_2}{\partial
t}+\vec v\cdot\nabla_{\vec
x}\bar\psi_2\right)\right)\\
+\frac{c}{2}\left(\left(\vec S\cdot\left(i\hbar\nabla_{\vec
x}\psi_1+\frac{\sigma}{c}\vec A(\vec
x,t)\psi_1\right)\right)\cdot\bar\psi_2-\psi_2\cdot\left(\vec
S\cdot\left(i\hbar\nabla_{\vec x}\bar\psi_1-\frac{\sigma}{c}\vec
A(\vec x,t)\bar\psi_1\right)\right) \right)
\\-mc^2\left(\psi_1\cdot\bar\psi_1-\psi_2\cdot\bar\psi_2\right)
-\sigma\left(\Psi-\frac{1}{c}\vec v\cdot\vec
A\right)\left(\psi_1\cdot\bar\psi_1+\psi_2\cdot\bar\psi_2\right)+V\left(\vec
x,t\right)\left(\psi_1\cdot\bar\psi_1+\psi_2\cdot\bar\psi_2\right)\\-\frac{\hbar}{4}\left(\vec
S\cdot\left(curl_{\vec x}\vec
v\right)\psi_1\right)\cdot\bar\psi_1-\frac{\hbar}{4}\left(\vec
S\cdot\left(curl_{\vec x}\vec
v\right)\psi_2\right)\cdot\bar\psi_2\\
+\frac{(g-1)\sigma\hbar}{2mc}\left(\vec S\cdot\left(curl_{\vec
x}\vec
A\right)\psi_1\right)\cdot\bar\psi_1+\frac{(g-1)\sigma\hbar}{2mc}\left(\vec
S\cdot\left(curl_{\vec x}\vec A\right)\psi_2\right)\cdot\bar\psi_2,
\end{multline}
%
%
%
%
%
%
where $\psi=(\psi_1,\psi_2)\in\mathbb{C}^2\times \mathbb{C}^2$ is a
four-component wave function. Then similarly to the proof of Theorem
\ref{gjghghgghgintintrrZZjkh} we can prove that $L$ is invariant
under the change of inertial or non-inertial cartesian coordinate
system, given by
\er{noninchgravortbstrjgghguittu2intrrrZZjkjjkjkjk}, provided that
we take into account
\er{yuythfgfyftydtydtydtyddyyyhhddhhhredPPN111hgghjgintintrrZZkjhh}.
We investigate stationary points of the functional
\begin{equation}\label{btfffygtgyggyDDmmkkkZZkkk}
J=\int_0^T\int_{\mathbb{R}^3}L\left(\psi,\vec x,t\right)d\vec x dt.
\end{equation}
Then, by \er{vhfffngghkjgghDDmmkkkZZkkk} we have
\begin{multline}\label{vhfffngghkjgghDDmmkkkghgZZkkk}
0=\frac{\delta L}{\delta (\bar\psi_1)}=
i\hbar\left(\frac{\partial\psi_1}{\partial t}+\frac{1}{2}\vec
v\cdot\nabla_{\vec x}\psi_1+\frac{1}{2}div_{\vec x}\left\{\psi_1\vec
v\right\}\right)+c\,\vec S\cdot\left(i\hbar\nabla_{\vec
x}\psi_2+\frac{\sigma}{c}\vec A(\vec x,t)\psi_2\right)-mc^2\psi_1
\\-\sigma\left(\Psi-\frac{1}{c}\vec v\cdot\vec A\right)\psi_1
-\frac{\hbar}{4}\,\vec S\cdot\left(curl_{\vec x}\vec
v\right)\psi_1+V(\vec x,t)\psi_1+\frac{(g-1)\sigma\hbar}{2mc},\vec
S\cdot\left(curl_{\vec x}\vec A\right)\psi_1+V(\vec x,t)\psi_1,
\end{multline}
\begin{multline}\label{vhfffngghkjgghDDmmkkkghgZZkkkhh}
0=\frac{\delta L}{\delta (\bar\psi_2)}=
i\hbar\left(\frac{\partial\psi_2}{\partial t}+\frac{1}{2}\vec
v\cdot\nabla_{\vec x}\psi_2+\frac{1}{2}div_{\vec x}\left\{\psi_2\vec
v\right\}\right)+c\,\vec S\cdot\left(i\hbar\nabla_{\vec
x}\psi_1+\frac{\sigma}{c}\vec A(\vec x,t)\psi_1\right)+mc^2\psi_2
\\-\sigma\left(\Psi-\frac{1}{c}\vec v\cdot\vec A\right)\psi_2
-\frac{\hbar}{4}\,\vec S\cdot\left(curl_{\vec x}\vec
v\right)\psi_1+V(\vec x,t)\psi_2+\frac{(g-1)\sigma\hbar}{2mc},\vec
S\cdot\left(curl_{\vec x}\vec A\right)\psi_2,
\end{multline}
\begin{multline}\label{vhfffngghkjgghDDmmkkkghguhyuyZZkkk}
0=\frac{\delta L}{\delta (\psi_1)}= \bar
{(i)}\hbar\left(\frac{\partial\bar\psi_1}{\partial
t}+\frac{1}{2}\vec v\cdot\nabla_{\vec
x}\bar\psi_1+\frac{1}{2}div_{\vec x}\left\{\bar\psi_1\vec
v\right\}\right)+c\,\vec S\cdot\left(\bar {(i)}\hbar\nabla_{\vec
x}\bar\psi_2+\frac{\sigma}{c}\vec A(\vec x,t)\bar\psi_2\right)
-mc^2\bar\psi_1
\\-\sigma\left(\Psi-\frac{1}{c}\vec v\cdot\vec A\right)\bar\psi_1
-\frac{\hbar}{4}\,\vec S^T\cdot\left(curl_{\vec x}\vec
v\right)\bar\psi_1+V(\vec
x,t)\bar\psi_1+\frac{(g-1)\sigma\hbar}{2mc},\vec
S^T\cdot\left(curl_{\vec x}\vec A\right)\bar\psi_1
\end{multline}
and
\begin{multline}\label{vhfffngghkjgghDDmmkkkghguhyuyZZkkkhh}
0=\frac{\delta L}{\delta (\psi_2)}= \bar
{(i)}\hbar\left(\frac{\partial\bar\psi_2}{\partial
t}+\frac{1}{2}\vec v\cdot\nabla_{\vec
x}\bar\psi_2+\frac{1}{2}div_{\vec x}\left\{\bar\psi_2\vec
v\right\}\right)+c\,\vec S\cdot\left(\bar {(i)}\hbar\nabla_{\vec
x}\bar\psi_1+\frac{\sigma}{c}\vec A(\vec x,t)\bar\psi_1\right)
+mc^2\bar\psi_2
\\-\sigma\left(\Psi-\frac{1}{c}\vec v\cdot\vec A\right)\bar\psi_2
-\frac{\hbar}{4}\,\vec S^T\cdot\left(curl_{\vec x}\vec
v\right)\bar\psi_2+V(\vec
x,t)\bar\psi_2+\frac{(g-1)\sigma\hbar}{2mc},\vec
S^T\cdot\left(curl_{\vec x}\vec A\right)\bar\psi_2.
\end{multline}
Note that \er{vhfffngghkjgghDDmmkkkghguhyuyZZkkk} and
\er{vhfffngghkjgghDDmmkkkghguhyuyZZkkkhh} are just the complex
conjugates of \er{vhfffngghkjgghDDmmkkkghgZZkkk} and
\er{vhfffngghkjgghDDmmkkkghgZZkkkhh}. So we get that the
Euler-Lagranges equation for \er{vhfffngghkjgghDDmmkkkZZkkk}
coincide with Dirac equations in the form of
\er{vhfffngghkjgghfjjghghghSYShmyuuiiuuhmhmiopoopnniukjhjkk;l;lkkkJJjjkjk}
and
\er{vhfffngghkjgghfjjghghghSYShmyuuiiuuhmhmiopoopnniukjhjkhhjl;kkJJkllkkl}.

\subsection{Classical Relativistic-like Lagrangian and Hamiltonian
of the motion of the system of $n$ particles}\label{hgggggjjhhk}
As
in
\er{btfffygtgyggyijhhkkSYSPNuhuygyygyggyuyyuyuykuhghgjjojiyyyuyuuyy}
and
\er{btfffygtgyggyijhhkkSYSPNuhuygyygyggyuyyuyuykuhghgjjojiyyyuyuuyyKK}
consider the relativistic-like Lagrangian of the motion of $n$
relativistic-like particles with inertial masses $m_1,\ldots,m_n$
and the charges $\sigma_1,\ldots,\sigma_n$ in the outer
gravitational and electromagnetical field with characteristics $\vec
v(\vec x,t)$, $\vec A(\vec x,t)$ and $\Psi(\vec x,t)$ and additional
conservative field with potential $V(\vec y_1,\ldots,\vec y_n,t)$:
\begin{multline}\label{btfffygtgyggyijhhkkSYSPNuhuygyygyggyuyyuyuykuhghgjjojiyyyuyuuyyKKkknew}
J_{rl}(\vec r_1,\ldots,\vec r_n)=\int\limits_0^T
L_0\left(\frac{d\vec
r_1}{dt},\ldots, \frac{d\vec r_n}{dt},\vec r_1,\ldots, \vec r_n,t\right)dt:=\\
\int\limits_0^T\left\{
-\sum_{j=1}^{n}m_jc^2\sqrt{1-\frac{1}{c^2}\left|\frac{d\vec
r_j}{dt}-\vec v(\vec
r_j,t)\right|^2}-\sum_{j=1}^{n}\sigma_j\left(\Psi(\vec
r_j,t)-\frac{1}{c}\vec A(\vec r_j,t)\cdot\frac{d\vec
r_j}{dt}\right)+V\left(\vec r_1,\ldots, \vec r_n,t\right)\right\}dt.
\end{multline}
Next define the generalized momentums of each particle by
\begin{multline}\label{guytyurtydftyiujhSYShmhmKKkknew}
\vec P_j:=\nabla_{\vec r'_j}L_0\left(\vec r'_1,\ldots, \vec
r'_n,\vec r_1,\ldots,\vec
r_n,t\right)\\=m_j\left(1-\frac{1}{c^2}\left|\frac{d\vec
r_j}{dt}-\vec v(\vec
r_j,t)\right|^2\right)^{-\frac{1}{2}}\left(\frac{d\vec r_j}{dt}-\vec
v(\vec r_j,t)\right)+\frac{\sigma_j}{c}\vec A(\vec r_j,t).
\end{multline}
Then
\begin{equation}\label{guytyurtydftyiujhioyhuyhSYShmhmKKkknew}
\left(1-\frac{1}{c^2}\left|\frac{d\vec r_j}{dt}-\vec v(\vec
r_j,t)\right|^2\right)^{-\frac{1}{2}}\left(\frac{d\vec r_j}{dt}-\vec
v(\vec r_j,t)\right)=\left(\frac{1}{m_j}\vec
P_j-\frac{\sigma_j}{m_jc}\vec A(\vec r_j,t)\right).
\end{equation}
So
\begin{equation}\label{guytyurtydftyiujhioyhuyhSYShmhmuiiuiuKKkknew}
\frac{d\vec r_j}{dt}=\left(1+\frac{1}{c^2}\left|\frac{1}{m_j}\vec
P_j-\frac{\sigma_j}{m_jc}\vec A(\vec
r_j,t)\right|^2\right)^{-\frac{1}{2}}\left(\frac{1}{m_j}\vec
P_j-\frac{\sigma_j}{m_jc}\vec A(\vec r_j,t)\right)+\vec v(\vec
r_j,t).
\end{equation}
Thus if we consider a Hamiltonian
\begin{equation}\label{vhfffngghkjgghfjjhyjjfgSYShmhmKKkknew}
H_0\left(\vec P_1,\ldots, \vec P_n,\vec r_1,\ldots, \vec
r_n,t\right):=\sum_{j=1}^{n}\vec P_j\cdot\frac{d\vec
r_j}{dt}-L_0\left(\frac{d\vec r_1}{dt},\ldots, \frac{d\vec
r_n}{dt},\vec r_1,\ldots, \vec r_n,t\right),
\end{equation}
then, as before in \er{vhfffngghkjgghfjjghghghSYShmyuuiiuuhmhmKK}
and \er{vhfffngghkjgghfjjghghghSYSPN} by
\er{vhfffngghkjgghfjjhyjjfgSYShmhmKKkknew},
\er{btfffygtgyggyijhhkkSYSPNuhuygyygyggyuyyuyuykuhghgjjojiyyyuyuuyyKKkknew},
\er{guytyurtydftyiujhioyhuyhSYShmhmKKkknew} and
\er{guytyurtydftyiujhioyhuyhSYShmhmuiiuiuKKkknew} we
%
%
%
%
%
%
obtain that the relativistic-like Hamiltonian for a macro-particles
has the form:
\begin{multline}\label{vhfffngghkjgghfjjghghghSYShmyuuiiuuhmhmiopoopKKkknew}
H_0\left(\vec P_1,\ldots, \vec P_n,\vec r_1,\ldots, \vec
r_n,t\right)=
\sum_{j=1}^{n}m_jc^2\left(1+\frac{1}{c^2}\left|\frac{1}{m_j}\vec
P_j-\frac{\sigma_j}{m_jc}\vec A(\vec
r_j,t)\right|^2\right)^{\frac{1}{2}}
\\+\sum_{j=1}^{n}\sigma_j\left(\Psi(\vec r_j,t)-\frac{1}{c}\vec
v(\vec r_j,t)\cdot\vec A(\vec r_j,t)\right)-V\left(\vec r_1,\ldots,
\vec r_n,t\right)+\sum_{j=1}^{n}\vec v(\vec r_j,t)\cdot\vec P_j.
\end{multline}

\subsection{Liouville's equation for a system of $n$ relativistic-like classical particles}\label{gghfhfhfhfh11}
Assume that the number of particles $n$ in the system, ruled by the
Hamiltonian
\er{vhfffngghkjgghfjjghghghSYShmyuuiiuuhmhmiopoopKKkknew}, is large
and we describe this system statistically. Then let $w\left(\vec
P_1,\ldots,\vec P_n,\vec r_1,\ldots,\vec r_n,t\right)\to[0,+\infty)$
be the probability density of the system which satisfies the well
known classical Liouville's equation of the form:
\begin{multline}\label{vhfffngghkjgghfjjghghghhjghjgghkghggSYSPNkljljjmjkj11}
\frac{\partial w}{\partial t}+\sum_{j=1}^{n}\left(div_{\vec
r_j}\left\{w\,\nabla_{\vec
P_j}H_0\right\}-div_{\vec P_j}\left\{w\,\nabla_{\vec r_j}H_0\right\}\right)=\\
\frac{\partial w}{\partial t}+\sum_{j=1}^{n}\left(\nabla_{\vec
P_j}H_0\cdot\nabla_{\vec r_j}w-\nabla_{\vec r_j}H_0\cdot\nabla_{\vec
P_j}w\right)=0.
\end{multline}
Then since by
\er{vhfffngghkjgghfjjghghghSYShmyuuiiuuhmhmiopoopKKkknew} we have
\begin{multline}\label{vhfffngghkjgghfjjghghghSYSPNjljl11}
H_0\left(\vec P_1,\ldots,\vec P_n,\vec r_1,\ldots,\vec
r_n,t\right)=\sum_{j=1}^{n}\vec P_j\cdot\vec v(\vec
r_j,t)+\sum_{j=1}^{n}m_jc^2\left(1+\frac{1}{c^2}\left|\frac{1}{m_j}\vec
P_j-\frac{\sigma_j}{m_jc}\vec A(\vec
r_j,t)\right|^2\right)^{\frac{1}{2}}\\
+\sum_{j=1}^{n}\sigma_j\left(\Psi(\vec r_j,t)-\frac{1}{c}\vec A(\vec
r_j,t)\cdot\vec v(\vec r_j,t)\right)-V\left(\vec r_1,\ldots,\vec
r_n,t\right),
\end{multline}
and in particular,
\begin{multline}\label{vhfffngghkjgghfjjghghghhjghjgghkghggSYSPNkljljjmjkjjjkkjlkjkljljjkjkkjjkl11}
\nabla_{\vec
P_j}H_0=\frac{1}{m_j}\left(1+\frac{1}{c^2}\left|\frac{1}{m_j}\vec
P_j-\frac{\sigma_j}{m_jc}\vec A(\vec
r_j,t)\right|^2\right)^{-\frac{1}{2}}\left(\vec
P_j-\frac{\sigma_j}{c}\vec A(\vec r_j,t)\right)+\vec v( \vec
r_j,t)\quad\text{and}\\ \nabla_{\vec r_j}H_0= \left\{d_{\vec x}\vec
v(\vec r_j,t)\right\}^T\cdot\vec P_j\\-\frac{\sigma_j}{c
m_j}\left(1+\frac{1}{c^2}\left|\frac{1}{m_j}\vec
P_j-\frac{\sigma_j}{m_jc}\vec A(\vec
r_j,t)\right|^2\right)^{-\frac{1}{2}}\left\{d_{\vec x}\vec A(\vec
r_j,t)\right\}^T\cdot\left(\vec P_j-\frac{\sigma_j}{c}\vec A(\vec
r_j,t)\right)\\+\sigma_j\nabla_{\vec x}\left(\Psi(\vec
r_j,t)-\frac{1}{c}\vec A(\vec r_j,t)\cdot\vec v(\vec
r_j,t)\right)-\nabla_{\vec y_j}V\left(\vec r_1,\ldots,\vec
r_n,t\right)\quad\quad\forall j=1,\ldots n,
\end{multline}
inserting
\er{vhfffngghkjgghfjjghghghhjghjgghkghggSYSPNkljljjmjkjjjkkjlkjkljljjkjkkjjkl11}
into \er{vhfffngghkjgghfjjghghghhjghjgghkghggSYSPNkljljjmjkj11}
gives
\begin{multline}\label{vhfffngghkjgghfjjghghghhjghjgghkghggSYSPNkljljjmjkjhhhjhigjgj11}
0=\frac{\partial w}{\partial t}+\sum_{j=1}^{n}\nabla_{\vec
P_j}H_0\cdot\nabla_{\vec r_j}w-\sum_{j=1}^{n}\nabla_{\vec
r_j}H_0\cdot\nabla_{\vec P_j}w= \frac{\partial w}{\partial t}
+\sum_{j=1}^{n}\vec v( \vec r_j,t)\cdot\nabla_{\vec r_j}w\\
+\sum_{j=1}^{n}\frac{1}{m_j}\left(1+\frac{1}{c^2}\left|\frac{1}{m_j}\vec
P_j-\frac{\sigma_j}{m_jc}\vec A(\vec
r_j,t)\right|^2\right)^{-\frac{1}{2}}\left(\vec
P_j-\frac{\sigma_j}{c}\vec A(\vec r_j,t)\right)\cdot\nabla_{\vec
r_j}w -\sum_{j=1}^{n}\left(\left\{d_{\vec x}\vec v(\vec
r_j,t)\right\}^T\cdot\vec P_j\right)\cdot\nabla_{\vec
P_j}w\\
+ \sum_{j=1}^{n}\left(\frac{\sigma_j}{c
m_j}\left(1+\frac{1}{c^2}\left|\frac{1}{m_j}\vec
P_j-\frac{\sigma_j}{m_jc}\vec A(\vec
r_j,t)\right|^2\right)^{-\frac{1}{2}}\left\{d_{\vec x}\vec A(\vec
r_j,t)\right\}^T\cdot\left(\vec P_j-\frac{\sigma_j}{c}\vec A(\vec
r_j,t)\right)\right)\cdot\nabla_{\vec P_j}w
\\
-\sum_{j=1}^{n}\left(\sigma_j\nabla_{\vec x}\left(\Psi(\vec
r_j,t)-\frac{1}{c}\vec A(\vec r_j,t)\cdot\vec v(\vec
r_j,t)\right)-\nabla_{\vec y_j}V\left(\vec r_1,\ldots,\vec
r_n,t\right)\right)\cdot\nabla_{\vec P_j}w = \frac{\partial
w}{\partial t} +\sum_{j=1}^{n}\vec v( \vec r_j,t)\cdot\nabla_{\vec
r_j}w
\\+\sum_{j=1}^{n}\frac{1}{m_j}\left(1+\frac{1}{c^2}\left|\frac{1}{m_j}\vec
P_j-\frac{\sigma_j}{m_jc}\vec A(\vec
r_j,t)\right|^2\right)^{-\frac{1}{2}}\left(\vec
P_j-\frac{\sigma_j}{c}\vec A(\vec r_j,t)\right)\cdot\nabla_{\vec
r_j}w \\+\sum_{j=1}^{n}\frac{1}{2}\left(\left(d_{\vec x}\vec v(\vec
r_j,t)-\left\{d_{\vec x}\vec v(\vec r_j,t)\right\}^T\right)\cdot\vec
P_j\right)\cdot\nabla_{\vec
P_j}w-\sum_{j=1}^{n}\frac{1}{2}\left(\left(d_{\vec x}\vec v(\vec
r_j,t)+\left\{d_{\vec x}\vec v(\vec r_j,t)\right\}^T\right)\cdot\vec
P_j\right)\cdot\nabla_{\vec
P_j}w\\
+ \sum_{j=1}^{n}\left(\frac{\sigma_j}{c
m_j}\left(1+\frac{1}{c^2}\left|\frac{1}{m_j}\vec
P_j-\frac{\sigma_j}{m_jc}\vec A(\vec
r_j,t)\right|^2\right)^{-\frac{1}{2}}\left\{d_{\vec x}\vec A(\vec
r_j,t)\right\}^T\cdot\left(\vec P_j-\frac{\sigma_j}{c}\vec A(\vec
r_j,t)\right)\right)\cdot\nabla_{\vec P_j}w
\\
-\sum_{j=1}^{n}\left(\sigma_j\nabla_{\vec x}\left(\Psi(\vec
r_j,t)-\frac{1}{c}\vec A(\vec r_j,t)\cdot\vec v(\vec
r_j,t)\right)-\nabla_{\vec y_j}V\left(\vec r_1,\ldots,\vec
r_n,t\right)\right)\cdot\nabla_{\vec P_j}w.
\end{multline}
Thus, by
\er{vhfffngghkjgghfjjghghghhjghjgghkghggSYSPNkljljjmjkjhhhjhigjgj11},
using \er{apfrm9}, we rewrite the Liouville's equation as:
\begin{multline}\label{vhfffngghkjgghfjjghghghhjghjgghkghggSYSPNkljljjmjkjhhhjhigjgjjhjgjg11}
\frac{\partial w}{\partial t} +\sum_{j=1}^{n}\vec v( \vec
r_j,t)\cdot\nabla_{\vec r_j}w
+\sum_{j=1}^{n}\frac{1}{2}\left(\left(curl_{\vec x}\vec v(\vec
r_j,t)\right)\times\vec P_j\right)\cdot\nabla_{\vec P_j}w
\\-\sum_{j=1}^{n}\frac{1}{2}\left(\left(d_{\vec x}\vec v(\vec
r_j,t)+\left\{d_{\vec x}\vec v(\vec r_j,t)\right\}^T\right)\cdot\vec
P_j\right)\cdot\nabla_{\vec P_j}w
\\+\sum_{j=1}^{n}\frac{1}{m_j}\left(1+\frac{1}{c^2}\left|\frac{1}{m_j}\vec
P_j-\frac{\sigma_j}{m_jc}\vec A(\vec
r_j,t)\right|^2\right)^{-\frac{1}{2}}\left(\vec
P_j-\frac{\sigma_j}{c}\vec A(\vec r_j,t)\right)\cdot\nabla_{\vec
r_j}w
\\+
\sum_{j=1}^{n}\left(\frac{\sigma_j}{c
m_j}\left(1+\frac{1}{c^2}\left|\frac{1}{m_j}\vec
P_j-\frac{\sigma_j}{m_jc}\vec A(\vec
r_j,t)\right|^2\right)^{-\frac{1}{2}}\left\{d_{\vec x}\vec A(\vec
r_j,t)\right\}^T\cdot\left(\vec P_j-\frac{\sigma_j}{c}\vec A(\vec
r_j,t)\right)\right)\cdot\nabla_{\vec P_j}w
\\
-\sum_{j=1}^{n}\left(\sigma_j\nabla_{\vec x}\left(\Psi(\vec
r_j,t)-\frac{1}{c}\vec A(\vec r_j,t)\cdot\vec v(\vec
r_j,t)\right)-\nabla_{\vec y_j}V\left(\vec r_1,\ldots,\vec
r_n,t\right)\right)\cdot\nabla_{\vec P_j}w=0.
\end{multline}
Next if the change of some non-inertial cartesian coordinate system
$(*)$ to another cartesian coordinate system $(**)$ is of the form
\er{noninchgravortbstrjgghguittu2}:
\begin{equation}\label{noninchgravortbstrjgghguittu2jjjk11}
\begin{cases}
\vec x'=A(t)\cdot\vec x+\vec z(t),\\
t'=t,
\end{cases}
\end{equation}
where $A(t)\in SO(3)$ is a rotation, then consistently with
\er{noninchgravortbstrjgghguittu2jjjk11} and
\er{guytyurtydftyiujhSYShmhmKKkknew} we have the following change of
variables $(\vec P_1,\ldots,\vec P_n,\vec x_1,\ldots,\vec
x_n,t)\to(\vec P'_1,\ldots,\vec P'_n,\vec x'_1,\ldots,\vec
x'_n,t')$:
\begin{equation}\label{noninchgravortbstrjgghguittu2intrrrZZHHjkjhjh11}
\begin{cases}
t'=t,\\
\vec x'_k=A(t)\cdot\vec x_k+\vec z(t)\quad\quad\forall
k=1,\ldots,n,\\
\vec P'_k=A(t)\cdot\vec P_k\quad\quad\forall k=1,\ldots,n.
\end{cases}
\end{equation}
Thus, since consistently with
\er{noninchgravortbstrjgghguittu2jjjk11} we have
\begin{equation}\label{yuythfgfyftydtydtydtyddyyyhhddhhhredPPN111hgghjgintintrrZZHHlll11}
\begin{cases}
V'\left(\vec x'_1,\ldots,\vec x'_n,t'\right)\,=\,V\left(\vec
x_1,\ldots,\vec x_n,t\right),\\
\sigma'_j\,=\,\sigma_j,\\
m'_j\,=\,m_j,\\
\vec v'(\vec x',t)\,=\,A(t)\cdot \vec v(\vec x,t)+\frac{dA}{dt}(t)\cdot\vec x+\frac{d\vec z}{dt}(t),\\
\vec A'(\vec x',t)\,=\,A(t)\cdot \vec A(\vec x,t),\\
\Psi'(\vec x',t)-\vec v'(\vec x',t)\cdot\vec A'(\vec
x',t)\,=\,\Psi(\vec x,t)-\vec v(\vec x,t)\cdot\vec A(\vec x,t),
\end{cases}
\end{equation}
by \er{noninchgravortbstrjgghguittu2intrrrZZHHjkjhjh11} and
\er{yuythfgfyftydtydtydtyddyyyhhddhhhredPPN111hgghjgintintrrZZHHlll11}
we deduce that the Liouville equation
\er{vhfffngghkjgghfjjghghghhjghjgghkghggSYSPNkljljjmjkjhhhjhigjgjjhjgjg11}
is invariant under the the change of non-inertial cartesian
coordinate system of the form
\er{noninchgravortbstrjgghguittu2jjjk11}.

\subsubsection{Thermodynamical equilibrium in the relativistic-like case}\label{hjgutgyutyytrr} Next let $w\left(\vec
P_1,\ldots,\vec P_n,\vec r_1,\ldots,\vec r_n,t\right)\to[0,+\infty)$
be the probability density of the system, ruled by the Hamiltonian
$H_0$ from \er{vhfffngghkjgghfjjghghghSYSPNjljl11}, in the case of
\underline{thermodynamical equilibrium}. Then in the case that
$\frac{\partial H_0}{\partial t}\equiv 0$ and the given equilibrium
system rests macroscopically, i.e. it has the macroscopical velocity
field zero: $\vec u(\vec x,t)\equiv 0$ it is well known that
\begin{equation}\label{vhfffngghkjgghfjjghghghSYSPNjljllkljjhkhkhklpl;lpoopkklkljlrr}
w\left(\vec P_1,\ldots,\vec P_n,\vec r_1,\ldots,\vec
r_n,t\right):=\frac{1}{K(f,H_{0})}f\left(H_{0}\left(\vec
P_1,\ldots,\vec P_n,\vec r_1,\ldots,\vec r_n,t\right)\right),
\end{equation}
with
\begin{equation}\label{vhfffngghkjgghfjjghghghSYSPNjljllkljjhkhkhklpl;liuiouiojjkrr}
K(f,H_{0}):=\int f\left(H_{0}\left(\vec P_1,\ldots,\vec P_n,\vec
r_1,\ldots,\vec r_n,t\right)\right)d\vec P_1\ldots d\vec P_nd\vec
r_1\ldots d\vec r_n,
\end{equation}
and,
\begin{itemize}
\item in the case of a system in thermostat, having the Kelvin temperature $T$, we have the canonical ensemble: $f(s)=e^{-\frac{s}{kT}}\;\,\forall s$, where $k$ is the Boltzmann
constant,
\item in the case of a thermally isolated system with the average internal energy $E$ we have the micro-canonical ensemble: $f(s)=\delta(s-E)\;\,\forall
s$.
\end{itemize}
We would like to find alternative forms of the above laws of
thermodynamical equilibrium, which are invariant under the change of
inertial or non-inertial cartesian coordinate system. Then it is
clear that if the given system rests in the old coordinate system,
then it obviously has non-trivial macroscopical velocity field $\vec
u(\vec x,t)$ in the new one. Moreover, $\vec u(\vec x,t)$ can depend
on $\vec x$ and $t$, as it indeed happens in the case of a rotation
of the new coordinate system with respect to the old one. On the
other hand the concept of thermodynamical equilibrium is clearly
independent of the coordinate system.

In order to find the forms of the laws of thermodynamical
equilibrium, which are indeed invariant under the change of inertial
or non-inertial cartesian coordinate system we follow the steps as
below. By \er{vhfffngghkjgghfjjghghghSYSPNjljl11} the Hamiltonian
$H_0$ of our system has the following form:
\begin{multline}\label{vhfffngghkjgghfjjghghghSYSPNjljll;kl;rr}
H_0\left(\vec P_1,\ldots,\vec P_n,\vec r_1,\ldots,\vec
r_n,t\right)=\sum_{j=1}^{n}\vec P_j\cdot\vec v(\vec
r_j,t)+\sum_{j=1}^{n}m_jc^2\left(1+\frac{1}{c^2}\left|\frac{1}{m_j}\vec
P_j-\frac{\sigma_j}{m_jc}\vec A(\vec
r_j,t)\right|^2\right)^{\frac{1}{2}}\\
+\sum_{j=1}^{n}\sigma_j\left(\Psi(\vec r_j,t)-\frac{1}{c}\vec A(\vec
r_j,t)\cdot\vec v(\vec r_j,t)\right)-V\left(\vec r_1,\ldots,\vec
r_n,t\right),
\end{multline}
Next define an invariant quantity $H_{\vec u}$ as:
\begin{multline}\label{vhfffngghkjgghfjjghghghSYSPNjljll'lrr}
H_{\vec u}\left(\vec P_1,\ldots,\vec P_n,\vec r_1,\ldots,\vec
r_n,t\right):=H_0\left(\vec P_1,\ldots,\vec P_n,\vec r_1,\ldots,\vec
r_n,t\right)-\sum_{j=1}^{n}\vec P_j\cdot\vec u(\vec r_j,t)\\=
\sum_{j=1}^{n}\vec P_j\cdot\left(\vec v(\vec r_j,t)-\vec u(\vec
r_j,t)\right)+
\sum_{j=1}^{n}m_jc^2\left(1+\frac{1}{c^2}\left|\frac{1}{m_j}\vec
P_j-\frac{\sigma_j}{m_jc}\vec A(\vec
r_j,t)\right|^2\right)^{\frac{1}{2}}\\+\sum_{j=1}^{n}\sigma_j\left(\Psi(\vec
r_j,t)-\frac{1}{c}\vec A(\vec r_j,t)\cdot\vec v(\vec
r_j,t)\right)-V\left(\vec r_1,\ldots,\vec r_n,t\right),
\end{multline}
where $\vec u(\vec x,t)$ is the macroscopical velocity field of the
given system of particles. Then, as before, it can be easily deduced
that under the change of some non-inertial cartesian coordinate
system $(*)$ to another cartesian coordinate system $(**)$ of the
form \er{noninchgravortbstrjgghguittu2jjjk11} the quantity $H_{\vec
u}$ transforms as:
\begin{equation}\label{vhfffngghkjgghfjjghghghSYSPNjljllkljjhkhkhklpl;lpoopkklkljljkhhjrr}
H'_{\vec u'}=H_{\vec u},
\end{equation}
provided that $\vec u$ is the speed-like vector field i.e. under the
above change we have:
\begin{equation}\label{yuythfgfyftydtydtydtyddyyyhhddhhhredPPN111hgghjgintintrrZZHHlllkljjjkrr}
\vec u'(\vec x',t)\,=\,A(t)\cdot \vec u(\vec
x,t)+\frac{dA}{dt}(t)\cdot\vec x+\frac{d\vec z}{dt}(t),
\end{equation}
and we have \er{noninchgravortbstrjgghguittu2intrrrZZHHjkjhjh11} and
\er{yuythfgfyftydtydtydtyddyyyhhddhhhredPPN111hgghjgintintrrZZHHlll11}.
Next consider the canonical and micro-canonical ensembles as:
\begin{equation}\label{vhfffngghkjgghfjjghghghSYSPNjljllkljjhkhkhklpl;lpooprr}
w\left(\vec P_1,\ldots,\vec P_n,\vec r_1,\ldots,\vec
r_n,t\right):=\frac{1}{K(f,H_{\vec u})}f\left(H_{\vec u}\left(\vec
P_1,\ldots,\vec P_n,\vec r_1,\ldots,\vec r_n,t\right)\right),
\end{equation}
where as before,
\begin{equation}\label{vhfffngghkjgghfjjghghghSYSPNjljllkljjhkhkhklpl;liuiouiorr}
K(f,H_{\vec u}):=\int f\left(H_{\vec u}\left(\vec P_1,\ldots,\vec
P_n,\vec r_1,\ldots,\vec r_n,t\right)\right)d\vec P_1\ldots d\vec
P_nd\vec r_1\ldots d\vec r_n,
\end{equation}
and,
\begin{itemize}
\item in the case of a system in thermostat, having the Kelvin temperature $T$, we have: $f(s)=e^{-\frac{s}{kT}}\;\,\forall
s$,
\item in the case of a thermally isolated system with the average internal energy $E$ we have: $f(s)=\delta(s-E)\;\,\forall
s$.
\end{itemize}
Note here, that by \er{vhfffngghkjgghfjjghghghSYSPNjljll'lrr} for
the finite integrability in
\er{vhfffngghkjgghfjjghghghSYSPNjljllkljjhkhkhklpl;liuiouiorr} the
vector field $\vec u(\vec x,t)$ must satisfy
\begin{equation}\label{btfffygtgyggyijhhkkSYSPNuhuygyygyggyuyyuyuykuhghgjjojiyyyuyuuyyijyyuyuhghghguy}
\left|\vec u(\vec x,t)-\vec v(\vec
x,t)\right|\,<\,c\quad\quad\quad\quad\forall\, (\vec x,t).
\end{equation}
Next, by
\er{vhfffngghkjgghfjjghghghSYSPNjljllkljjhkhkhklpl;lpoopkklkljljkhhjrr}
we deduce that under the change of some non-inertial cartesian
coordinate system $(*)$ to another cartesian coordinate system
$(**)$ of the form \er{noninchgravortbstrjgghguittu2jjjk11} the
quantity $w$ transforms as:
\begin{equation}\label{vhfffngghkjgghfjjghghghSYSPNjljllkljjhkhkhklpl;lpoopkklkljljkhhjkhkhkrr}
w'=w.
\end{equation}
Moreover, in the case $\vec u\equiv 0$
\er{vhfffngghkjgghfjjghghghSYSPNjljllkljjhkhkhklpl;lpooprr}
coincides with
\er{vhfffngghkjgghfjjghghghSYSPNjljllkljjhkhkhklpl;lpoopkklkljlrr}.
However, we still need to derive the restrictions on the field $\vec
u$ and the Hamiltonian $H_0$, providing that our system can indeed
be found in the state of thermodynamical equilibrium. We remind that
in the case $\vec u\equiv 0$ the appropriate restriction is
$\frac{\partial H_0}{\partial t}\equiv 0$. In order to get these
restrictions in the general case, we need to insert $w$ in
\er{vhfffngghkjgghfjjghghghSYSPNjljllkljjhkhkhklpl;lpooprr} into the
Liouville's equation in
\er{vhfffngghkjgghfjjghghghhjghjgghkghggSYSPNkljljjmjkjhhhjhigjgj11}
having the form:
\begin{equation}\label{vhfffngghkjgghfjjghghghhjghjgghkghggSYSPNkljljjmjkjkjjkrr}
\frac{\partial w}{\partial t}+\sum_{j=1}^{n}\left(\nabla_{\vec
P_j}H_0\cdot\nabla_{\vec r_j}w-\nabla_{\vec r_j}H_0\cdot\nabla_{\vec
P_j}w\right)=0.
\end{equation}
By \er{vhfffngghkjgghfjjghghghSYSPNjljll'lrr} we have:
\begin{multline}\label{vhfffngghkjgghfjjghghghhjghjgghkghggSYSPNkljljjmjkjjjkkjlkjkljljjkjkkjjklkjjkjhrr}
\nabla_{\vec P_j}H_{\vec u}=\nabla_{\vec P_j}H_0-\vec u( \vec
r_j,t)\quad\text{and}\quad \nabla_{\vec r_j}H_{\vec u}=\nabla_{\vec
r_j}H_0-\left\{d_{\vec x}\vec u(\vec r_j,t)\right\}^T\cdot\vec
P_j\quad\quad\forall j=1,\ldots n.
\end{multline}
Next inserting
\er{vhfffngghkjgghfjjghghghSYSPNjljllkljjhkhkhklpl;lpooprr} into
\er{vhfffngghkjgghfjjghghghhjghjgghkghggSYSPNkljljjmjkjkjjkrr} we
deuce:
\begin{equation}\label{vhfffngghkjgghfjjghghghhjghjgghkghggSYSPNkljljjkklmjkjkjjkrr}
\frac{\partial H_{\vec u}}{\partial
t}+\sum_{j=1}^{n}\left(\nabla_{\vec P_j}H_0\cdot\nabla_{\vec
r_j}H_{\vec u}-\nabla_{\vec r_j}H_0\cdot\nabla_{\vec P_j}H_{\vec
u}\right)=0.
\end{equation}
Thus inserting
\er{vhfffngghkjgghfjjghghghhjghjgghkghggSYSPNkljljjmjkjjjkkjlkjkljljjkjkkjjklkjjkjhrr}
into
\er{vhfffngghkjgghfjjghghghhjghjgghkghggSYSPNkljljjkklmjkjkjjkrr} we
obtain,
\begin{multline}\label{vhfffngghkjgghfjjghghghhjghjgghkghggSYSPNkljljjkklmjkjkjjk,mmm,kl;k;lklkrr}
0=\frac{\partial H_{\vec u}}{\partial
t}+\sum_{j=1}^{n}\left(\left(\nabla_{\vec P_j}H_{\vec u}+\vec u(
\vec r_j,t)\right)\cdot\nabla_{\vec r_j}H_{\vec
u}-\left(\nabla_{\vec r_j}H_{\vec u}+\left\{d_{\vec x}\vec u(\vec
r_j,t)\right\}^T\cdot\vec P_j\right)\cdot\nabla_{\vec P_j}H_{\vec
u}\right)=\\
\frac{\partial H_{\vec u}}{\partial t}+\sum_{j=1}^{n}\vec u( \vec
r_j,t)\cdot\nabla_{\vec r_j}H_{\vec
u}-\sum_{j=1}^{n}\left(\left\{d_{\vec x}\vec u(\vec
r_j,t)\right\}^T\cdot\vec P_j\right)\cdot\nabla_{\vec P_j}H_{\vec u}
=\\
\frac{\partial H_{\vec u}}{\partial t}+\sum_{j=1}^{n}\vec u( \vec
r_j,t)\cdot\nabla_{\vec r_j}H_{\vec
u}+\sum_{j=1}^{n}\frac{1}{2}\left(\left(d_{\vec x}\vec u(\vec
r_j,t)-\left\{d_{\vec x}\vec u(\vec r_j,t)\right\}^T\right)\cdot\vec
P_j\right)\cdot\nabla_{\vec P_j}H_{\vec
u}\\-\sum_{j=1}^{n}\frac{1}{2}\left(\left(d_{\vec x}\vec u(\vec
r_j,t)+\left\{d_{\vec x}\vec u(\vec r_j,t)\right\}^T\right)\cdot\vec
P_j\right)\cdot\nabla_{\vec P_j}H_{\vec u}.
\end{multline}
Thus, by \er{apfrm9}, we rewrite
\er{vhfffngghkjgghfjjghghghhjghjgghkghggSYSPNkljljjkklmjkjkjjk,mmm,kl;k;lklkrr}
as:
\begin{multline}\label{vhfffngghkjgghfjjghghghhjghjgghkghggSYSPNkljljjkklmjkjkjjk,mmm,kl;k;lklkaarr}
\frac{\partial H_{\vec u}}{\partial t}+\sum_{j=1}^{n}\vec u( \vec
r_j,t)\cdot\nabla_{\vec r_j}H_{\vec
u}+\sum_{j=1}^{n}\frac{1}{2}\left(\left(curl_{\vec x}\vec u(\vec
r_j,t)\right)\times\vec P_j\right)\cdot\nabla_{\vec P_j}H_{\vec
u}\\-\sum_{j=1}^{n}\frac{1}{2}\left(\left(d_{\vec x}\vec u(\vec
r_j,t)+\left\{d_{\vec x}\vec u(\vec r_j,t)\right\}^T\right)\cdot\vec
P_j\right)\cdot\nabla_{\vec P_j}H_{\vec u}=0.
\end{multline}
Equality
\er{vhfffngghkjgghfjjghghghhjghjgghkghggSYSPNkljljjkklmjkjkjjk,mmm,kl;k;lklkaarr}
is the required restriction on the field $\vec u$ and the
Hamiltonian $H_{0}$, providing that our system can indeed be found
in state of thermodynamical equilibrium. In particular, if
$\frac{\partial H_0}{\partial t}=0$ and $\vec u=0$ then
\er{vhfffngghkjgghfjjghghghhjghjgghkghggSYSPNkljljjkklmjkjkjjk,mmm,kl;k;lklkaarr}
indeed holds. Moreover, as before in the case of Liouville's
equation, we deduce that the equation
\er{vhfffngghkjgghfjjghghghhjghjgghkghggSYSPNkljljjkklmjkjkjjk,mmm,kl;k;lklkaarr}
is invariant under the change of non-inertial cartesian coordinate
system of the form \er{noninchgravortbstrjgghguittu2jjjk11},
provided that we have
\er{yuythfgfyftydtydtydtyddyyyhhddhhhredPPN111hgghjgintintrrZZHHlllkljjjkrr},
\er{noninchgravortbstrjgghguittu2intrrrZZHHjkjhjh11} and
\er{yuythfgfyftydtydtydtyddyyyhhddhhhredPPN111hgghjgintintrrZZHHlll11}.
Finally, we still need to prove that if $w$ is given by
\er{vhfffngghkjgghfjjghghghSYSPNjljllkljjhkhkhklpl;lpooprr} then the
vector field $\vec u(\vec x,t)$ is indeed the macroscopic (average)
velocity field of every particle that we can found near the point
$\vec x$ at the instant of time $t$. Indeed, we need to prove the
following:
\begin{multline}\label{guytyurtydftyiujhioyhuyhSYSPNlkjljklklrr}
\vec u(\vec x,t):=\frac{m_j}{\mu_j(\vec x,t)}\int \left(\nabla_{\vec
P_j}H_0\right)w\left(\vec P_1,\ldots,\vec P_n,\vec r_1,\ldots,\vec
r_n,t\right)\delta(\vec x-\vec r_j)\,d \vec P_1,\ldots,d\vec
P_n,d\vec r_1,\ldots,d\vec r_n,
\end{multline}
where
\begin{equation}\label{vhfffngghkjgghfjjghghghhjghjgghkghggSYSPNkljljjmjkjjjkkjlkjkllkljkmjhhhkghggklklnnmmmhjjhrr}
\mu_j(\vec x,t):=\int m_j \,w\left(\vec P_1,\ldots,\vec P_n,\vec
r_1,\ldots,\vec r_n,t\right)\delta(\vec x-\vec r_j)\,d \vec
P_1,\ldots,d\vec P_n,d\vec r_1,\ldots,d\vec r_n.
\end{equation}
On the other hand, by
\er{vhfffngghkjgghfjjghghghhjghjgghkghggSYSPNkljljjmjkjjjkkjlkjkljljjkjkkjjklkjjkjhrr}
we have
\begin{equation}\label{vhfffngghkjgghfjjghghghhjghjgghkghggSYSPNkljljjmjkjjjkkjlkjkljljjkjkkjjklkjjkjhrrkokkljk}
\nabla_{\vec P_j}H_0=\nabla_{\vec P_j}H_{\vec u}+\vec u( \vec
r_j,t)\quad\quad\forall j=1,\ldots n.
\end{equation}
Thus, by
\er{vhfffngghkjgghfjjghghghhjghjgghkghggSYSPNkljljjmjkjjjkkjlkjkljljjkjkkjjklkjjkjhrrkokkljk}
and \er{vhfffngghkjgghfjjghghghSYSPNjljllkljjhkhkhklpl;lpooprr} we
deduce,
\begin{multline}\label{vhfffngghkjgghfjjghghghhjghjgghkghggSYSPNkljljjmjkjjjkkjlkjkljljjkjkkjjklkjjkjhrrkokklklklklj}
\int \left(\nabla_{\vec P_j}H_0\right)w\left(\vec P_1,\ldots,\vec
P_n,\vec r_1,\ldots,\vec r_n,t\right)\delta(\vec x-\vec r_j)\,d \vec
P_1,\ldots,d\vec P_n,d\vec r_1,\ldots,d\vec r_n=\\
\frac{1}{K(f,H_{\vec u})}\int \left(\nabla_{\vec P_j}H_{\vec
u}\right)f\left(H_{\vec u}\left(\vec P_1,\ldots,\vec P_n,\vec
r_1,\ldots,\vec r_n,t\right)\right)\delta(\vec x-\vec r_j)\,d \vec
P_1,\ldots,d\vec P_n,d\vec r_1,\ldots,d\vec r_n+\\
\int \vec u( \vec r_j,t)w\left(\vec P_1,\ldots,\vec P_n,\vec
r_1,\ldots,\vec r_n,t\right)\delta(\vec x-\vec r_j)\,d \vec
P_1,\ldots,d\vec P_n,d\vec r_1,\ldots,d\vec r_n=\\
\frac{1}{K(f,H_{\vec u})}\int \left(\nabla_{\vec
P_j}\left(F\left(H_{\vec u}\left(\vec P_1,\ldots,\vec P_n,\vec
r_1,\ldots,\vec r_n,t\right)\right)\right)\right)\delta(\vec x-\vec
r_j)\,d \vec P_1,\ldots,d\vec P_n,d\vec r_1,\ldots,d\vec
r_n\\+\frac{\mu_j(\vec x,t)}{m_j}\vec u( x,t)= 0+\frac{\mu_j(\vec
x,t)}{m_j}\vec u( x,t),
\end{multline}
where $F(s):=\int_0^s f(\tau)d\tau$. So, by
\er{vhfffngghkjgghfjjghghghhjghjgghkghggSYSPNkljljjmjkjjjkkjlkjkljljjkjkkjjklkjjkjhrrkokklklklklj}
we indeed prove \er{guytyurtydftyiujhioyhuyhSYSPNlkjljklklrr}.

Next consider the following conditions
\begin{equation}\label{vhfffngghkjgghfjjghghghhjghjgghkghggSYSPNkljljjkklmjkjkjjk,mmm,kl;k;lklkjkjkkjkjkjkjljljjljklggghgkhhkhkhkkklkklklkljkhkhjjkhjh22}
\begin{cases}
\frac{\partial U}{\partial t}\left(\vec x_1,\ldots,\vec
x_n,t\right)+\sum_{j=1}^{n}\vec u( \vec x_j,t)\cdot\nabla_{\vec
x_j}U\left(\vec x_1,\ldots,\vec x_n,t\right)=0,\\
\frac{\partial}{\partial t}\left(\vec u(\vec x,t)-\vec v(\vec
x,t)\right)-curl_{\vec x}\left(\vec u(\vec x,t)\times\left(\vec
u(\vec x,t)-\vec v(\vec x,t)\right)\right)+\left(\Div_{\vec
x}\left(\vec u(\vec x,t)-\vec v(\vec x,t)\right)\right)\vec
u(\vec x,t)=0,\\
\frac{\partial\vec A}{\partial t}(\vec x,t)-curl_{\vec x}\left(\vec
u(\vec x,t)\times\vec A(\vec x,t)\right)+\left(\Div_{\vec x}\vec
A(\vec x,t)\right)\vec u(\vec x,t)=0,
\\
d_{\vec x}\vec u(\vec x,t)+\left\{d_{\vec x}\vec u(\vec
x,t)\right\}^T=0,
\end{cases}
\end{equation}
where
\begin{multline}\label{vhfffngghkjgghfjjghghghSYSPNjljll'lhkhhjgggjgjgghgjkjjk22}
U\left(\vec x_1,\ldots,\vec x_n,t\right):=\\
\sum_{j=1}^{n}\left(\frac{\sigma^2_j}{2m_jc^2}\left|\vec A(\vec
x_j,t)\right|^2+\sigma_j\Psi(\vec x_j,t)-\frac{\sigma_j}{c}\vec
A(\vec x_j,t)\cdot\vec v(\vec x_j,t)\right)-V\left(\vec
x_1,\ldots,\vec x_n,t\right),
\end{multline}
that are the same conditions as in
\er{vhfffngghkjgghfjjghghghhjghjgghkghggSYSPNkljljjkklmjkjkjjk,mmm,kl;k;lklkjkjkkjkjkjkjljljjljklggghgkhhkhkhkkklkklklkljkhkhjjkhjh},
\er{vhfffngghkjgghfjjghghghSYSPNjljll'lhkhhjgggjgjgghgjkjjk} and
were the necessary conditions for the Thermodynamical equilibrium in
the fully non-relativistic case. Moreover, as before, equations
\er{vhfffngghkjgghfjjghghghhjghjgghkghggSYSPNkljljjkklmjkjkjjk,mmm,kl;k;lklkjkjkkjkjkjkjljljjljklggghgkhhkhkhkkklkklklkljkhkhjjkhjh22},
\er{vhfffngghkjgghfjjghghghSYSPNjljll'lhkhhjgggjgjgghgjkjjk22} are
invariant under the change of inertial or non-inertial cartesian
coordinate systems. We will prove now that, as in the fully
non-relativistic case, conditions
\er{vhfffngghkjgghfjjghghghhjghjgghkghggSYSPNkljljjkklmjkjkjjk,mmm,kl;k;lklkjkjkkjkjkjkjljljjljklggghgkhhkhkhkkklkklklkljkhkhjjkhjh22},
\er{vhfffngghkjgghfjjghghghSYSPNjljll'lhkhhjgggjgjgghgjkjjk22} imply
\er{vhfffngghkjgghfjjghghghhjghjgghkghggSYSPNkljljjkklmjkjkjjk,mmm,kl;k;lklkrr}
or equivalently
\er{vhfffngghkjgghfjjghghghhjghjgghkghggSYSPNkljljjkklmjkjkjjk,mmm,kl;k;lklkaarr}
also in the case of the relativistic-like Hamiltonian
\er{vhfffngghkjgghfjjghghghSYSPNjljll;kl;rr} and the quantity
$H_{\vec u}$ in \er{vhfffngghkjgghfjjghghghSYSPNjljll'lrr}.

Indeed, by
\er{vhfffngghkjgghfjjghghghhjghjgghkghggSYSPNkljljjkklmjkjkjjk,mmm,kl;k;lklkjkjkkjkjkjkjljljjljklggghgkhhkhkhkkklkklklkljkhkhjjkhjh22}
and Proposition \ref{jhjkhhjhjhgjgjjkjk} there exists another
cartesian coordinate system $(**)$ such that under the change of
coordinate system $(*)$ to another cartesian coordinate system
$(**)$, given by \er{noninchgravortbstrjgghguittu2hhhhh}, we have
\begin{equation}
\label{jljjjjkhhhjkk}
A(t)\cdot \vec u(\vec x,t)+
\frac{d A}{dt}(t)\cdot\vec x+
\frac{d\vec z}{dt}(t)=\vec u'(\vec x',t')=0.
\end{equation}
Thus, since
\er{vhfffngghkjgghfjjghghghhjghjgghkghggSYSPNkljljjkklmjkjkjjk,mmm,kl;k;lklkjkjkkjkjkjkjljljjljklggghgkhhkhkhkkklkklklkljkhkhjjkhjh22},
\er{vhfffngghkjgghfjjghghghSYSPNjljll'lhkhhjgggjgjgghgjkjjk22} are
invariant under the change of inertial or non-inertial cartesian
coordinate systems, in system $(**)$ we have
\begin{equation}\label{vhfffngghkjgghfjjghghghhjghjgghkghggSYSPNkljljjkklmjkjkjjk,mmm,kl;k;lklkjkjkkjkjkjkjljljjljklggghgkhhkhkhkkklkklklkljkhkhjjkhjhmjmm,mjhjkjk22}
\begin{cases}
\frac{\partial U'}{\partial t'}\left(\vec x'_1,\ldots,\vec
x'_n,t'\right)=0,\\
\frac{\partial\vec v'}{\partial t'}(\vec
x',t')=0,\\
\frac{\partial\vec A'}{\partial t'}(\vec x',t')=0
\\
\vec u'(\vec x',t')=0,
\end{cases}
\end{equation}
where
\begin{multline}\label{vhfffngghkjgghfjjghghghSYSPNjljll'lhkhhjgggjgjgghgjkjjkkj22}
U'\left(\vec x'_1,\ldots,\vec x'_n,t'\right):=\\
\sum_{j=1}^{n}\left(\frac{(\sigma'_j)^2}{2m'_jc^2}\left|\vec A'(\vec
x'_j,t')\right|^2+\sigma'_j\Psi'(\vec
x'_j,t')-\frac{\sigma'_j}{c}\vec A'(\vec x'_j,t')\cdot\vec v'(\vec
x'_j,t')\right)-V'\left(\vec x'_1,\ldots,\vec x'_n,t'\right).
\end{multline}
On the other hand,
\er{vhfffngghkjgghfjjghghghhjghjgghkghggSYSPNkljljjkklmjkjkjjk,mmm,kl;k;lklkjkjkkjkjkjkjljljjljklggghgkhhkhkhkkklkklklkljkhkhjjkhjhmjmm,mjhjkjk22}
is equivalent to
\begin{equation}\label{vhfffnhhjhjjkkjhk22}
\frac{\partial H'_{0}}{\partial t'}\left(\vec P'_1,\ldots,\vec
P'_n,\vec r'_1,\ldots,\vec r'_n,t'\right)=0\quad\text{and}\quad\vec
u'(\vec x',t')=0,
\end{equation}
where $H_0$ is given by
\er{vhfffngghkjgghfjjghghghSYSPNjljll;kl;rr}. Then
\er{vhfffnhhjhjjkkjhk22} implies that we have the primed version of
\er{vhfffngghkjgghfjjghghghhjghjgghkghggSYSPNkljljjkklmjkjkjjk,mmm,kl;k;lklkaarr}
and therefore, since
\er{vhfffngghkjgghfjjghghghhjghjgghkghggSYSPNkljljjkklmjkjkjjk,mmm,kl;k;lklkaarr}
is invariant under the change of cartesian coordinate systems obtain
also the original version of
\er{vhfffngghkjgghfjjghghghhjghjgghkghggSYSPNkljljjkklmjkjkjjk,mmm,kl;k;lklkaarr}.
So,
\er{vhfffngghkjgghfjjghghghhjghjgghkghggSYSPNkljljjkklmjkjkjjk,mmm,kl;k;lklkjkjkkjkjkjkjljljjljklggghgkhhkhkhkkklkklklkljkhkhjjkhjh22},
\er{vhfffngghkjgghfjjghghghSYSPNjljll'lhkhhjgggjgjgghgjkjjk22} imply
\er{vhfffngghkjgghfjjghghghhjghjgghkghggSYSPNkljljjkklmjkjkjjk,mmm,kl;k;lklkaarr}
or equivalently
\er{vhfffngghkjgghfjjghghghhjghjgghkghggSYSPNkljljjkklmjkjkjjk,mmm,kl;k;lklkrr}.
Thus, we again obtain that the condition for existing of the
possibility to achieve the Thermodynamical equilibrium in the given
system, ruled by the Hamiltonian $H_0\left(\vec P_1,\ldots,\vec
P_n,\vec r_1,\ldots,\vec r_n,t\right)$ of the form
\er{vhfffngghkjgghfjjghghghSYSPNjljll;kl;rr}, is equivalent to the
existence of a cartesian coordinate system $(**)$, where the average
velocity of every particle vanishes and at the same time in the
system $(**)$ the Hamiltonian is independent of the time variable
explicitly.

\subsection{Relativistic-like Dirac equation for a system of $n$
spin-half micro-particles}\label{ghghggjghfgjknew} Consider the
motion of a system of $n$ spin-half relativistic-like stable
micro-particles with inertial masses $m_1,\ldots,m_n$ and the
charges $\sigma_1,\ldots,\sigma_n$ in the outer gravitational and
electromagnetical field with characteristics $\vec v(\vec x,t)$,
$\vec A(\vec x,t)$ and $\Psi(\vec x,t)$ and additional conservative
field with potential $V(\vec y_1,\ldots,\vec y_n,t)$, taking into
account the spin interaction. Then the system is characterized by
$4^n$-component complex wave function $\psi(\vec x_1,\ldots,\vec
x_n,t)\in\mathbb{C}^{4^n}$ where by $\mathbb{C}^{4^n}$ we denote the
tensor product of $n$ copies of the space $\mathbb{C}^2\otimes
\mathbb{C}^2=\mathbb{C}^4$:
\begin{equation}\label{fyfgjguyiuiHHnew}
\mathbb{C}^{4^n}:=\left(\mathbb{C}^4\right)\otimes_1\left(\mathbb{C}^4\right)\otimes_2\left(\mathbb{C}^4\right)
\ldots\otimes_{(n-1)}\left(\mathbb{C}^4\right),
\end{equation}
and given $a=(a_1,a_2)\in\mathbb{C}^2$ and
$b=(b_1,b_2)\in\mathbb{C}^2$ we identify their tensor product
$a\otimes b\in\mathbb{C}^2\otimes \mathbb{C}^2$ with the vector
$(a_1,a_2,b_1,b_2)\in\mathbb{C}^4$. Then, as before in
\er{vhfffngghkjgghfjjghghghSYShmyuuiiuuhmhmiopoopnniukjhjkkkJJ},\er{vhfffngghkjgghfjjghghghSYShmyuuiiuuhmhmiopoopnniukjhjkhhjkkJJ}
and in \er{vhfffngghkjgghfjjghghghhjghjgghkghgghjhggjjkgSYSPNHH}
consistently with the classical Hamiltonian
\er{vhfffngghkjgghfjjghghghSYShmyuuiiuuhmhmiopoopKKkknew}, we built
the Hermitian Hamiltonian operator as:
\begin{multline}\label{vhfffngghkjgghfjjghghghhjghjgghkghgghjhggjjkgSYSPNHHnew} \hat
H_0\cdot\psi=\\
\sum_{j=1}^{n}m_jc^2D^j_4\cdot\psi-\sum_{j=1}^{n}c\,\vec
D_j\cdot\left(i\hbar\nabla_{\vec x_j}\psi+\frac{\sigma_j}{c}\vec
A(\vec x_j,t)\psi\right)
%
%
%
+\sum_{j=1}^{n}\sigma_j\left(\Psi(\vec x_j,t)-\frac{1}{c}\vec v(\vec
x_j,t)\cdot\vec A(\vec x_j,t)\right)\psi\\-V\left(\vec
x_1,\ldots,\vec
x_n,t\right)\psi-\sum_{j=1}^{n}\frac{i\hbar}{2}div_{\vec
x_j}\left\{\psi\vec v(\vec
x_j,t)\right\}-\sum_{j=1}^{n}\frac{i\hbar}{2}\nabla_{\vec
x_j}\psi\cdot\vec v(\vec x_j,t)+\sum_{j=1}^{n}\frac{\hbar}{4}\vec
M_j\cdot\left(curl_{\vec x_j}\vec v(\vec x_j,t)\,\psi\right)
\\-\sum_{j=1}^{n}\frac{(g_j-1)\sigma_j\hbar}{2m_jc}\vec
M_j\cdot\left(curl_{\vec x_j}\vec A(\vec x_j,t)\,\psi\right),
\end{multline}
where $\hat H_0$ is the Hamiltonian operator, for every
$j=1,\ldots,n$ $g_j$ is a constant that depends on the type of the
particle (for electron we have $g_j=1$) and for every $j=1,\ldots,
n$ we denote
\begin{equation}\label{fyfgjguyiuiHHHHHHnew}
\vec D_j:=(D^j_1,D^j_2,D^j_3)\quad\text{and}\quad\vec
M_j:=(M^j_1,M^j_2,M^j_3)\quad\quad\forall\, j=1,2,\ldots,n,
\end{equation}
where for every $k=1,2,3,4$ and every $j=1,2,\ldots,n$:
$D^j_k:\mathbb{C}^{4^n}\to \mathbb{C}^{4^n}$ is a linear operator on
$\mathbb{C}^{4^n}$ (i.e. it is a $4^n\times 4^n$-complex matrix)
defined by the following identities:
\begin{multline}\label{fyfgjguyiuiHHHHHHHHHHuiunew}
D^1_k:=\left(D_k\right)\otimes_{1}\left(I^{4\times
4}\right)\otimes_{2}\left(I^{4\times
4}\right)\ldots\otimes_{(n-1)}\left(I^{4\times
4}\right),\quad\quad\ldots\quad
\\
D^j_k:=\left(I^{4\times 4}\right)\otimes_1\left(I^{4\times
4}\right)\otimes_2\left(I^{4\times 4}\right)
\ldots\otimes_{(j-1)}\left(D_k\right)\otimes_{j}\left(I^{4\times
4}\right)\otimes_{(j+1)}\left(I^{4\times
4}\right)\ldots\otimes_{(n-1)}\left(I^{4\times 4}\right),
\\
\quad\quad\ldots\quad\text{and}\quad\quad D^n_k:=\left(I^{4\times
4}\right)\otimes_1\left(I^{4\times
4}\right)\otimes_2\left(I^{4\times 4}\right)
\ldots\otimes_{(n-2)}\left(I^{4\times
4}\right)\otimes_{(n-1)}\left(D_k\right),
\end{multline}
and for every $k=1,2,3$ and every $j=1,2,\ldots,n$:
$M^j_k:\mathbb{C}^{4^n}\to \mathbb{C}^{4^n}$ is a linear operator on
$\mathbb{C}^{4^n}$ (i.e. it is a $4^n\times 4^n$-complex matrix)
defined by the following identities:
\begin{multline}\label{fyfgjguyiuiHHHHHHHHHHuiuhnhgbbnew}
M^1_k:=\left(M_k\right)\otimes_{1}\left(I^{4\times
4}\right)\otimes_{2}\left(I^{4\times
4}\right)\ldots\otimes_{(n-1)}\left(I^{4\times
4}\right),\quad\quad\ldots\quad
\\
M^j_k:=\left(I^{4\times 4}\right)\otimes_1\left(I^{4\times
4}\right)\otimes_2\left(I^{4\times 4}\right)
\ldots\otimes_{(j-1)}\left(M_k\right)\otimes_{j}\left(I^{4\times
4}\right)\otimes_{(j+1)}\left(I^{4\times
4}\right)\ldots\otimes_{(n-1)}\left(I^{4\times 4}\right),
\\
\quad\quad\ldots\quad\text{and}\quad\quad M^n_k:=\left(I^{4\times
4}\right)\otimes_1\left(I^{4\times
4}\right)\otimes_2\left(I^{4\times 4}\right)
\ldots\otimes_{(n-2)}\left(I^{4\times
4}\right)\otimes_{(n-1)}\left(M_k\right).
\end{multline}
Here $D_k$ for $k=1,2,3,4$ are $4\times 4$ Dirac $\alpha$-matrixes
defined as:
\begin{multline}\label{ghgyfftdtdnew}
D_1=\left(\begin{matrix}O^{2\times 2}&S_1\\S_1&O^{2\times 2}\end{matrix}\right)=\left(\begin{matrix}0&0&0&1\\0&0&1&0\\
0&1&0&0\\1&0&0&0\end{matrix}\right),\quad
D_2=\left(\begin{matrix}O^{2\times 2}&S_2\\S_2&O^{2\times 2}\end{matrix}\right)=\left(\begin{matrix}0&0&0&-i\\0&0&i&0\\
0&-i&0&0\\i&0&0&0\end{matrix}\right),\\
D_3=\left(\begin{matrix}O^{2\times 2}&S_3\\S_3&O^{2\times 2}\end{matrix}\right)=\left(\begin{matrix}0&0&1&0\\0&0&0&-1\\
1&0&0&0\\0&-1&0&0\end{matrix}\right), \quad
D_4=\left(\begin{matrix}I^{2\times 2}&O^{2\times 2}\\O^{2\times 2}&-I^{2\times 2}\end{matrix}\right)=\left(\begin{matrix}1&0&0&0\\0&1&0&0\\
0&0&-1&0\\0&0&0&-1\end{matrix}\right),
\end{multline}
and $M_k$ for $k=1,2,3,4$ are $4\times 4$ matrixes defined as:
\begin{multline}\label{ghgyfftdtdljkljnnew}
M_1=\left(\begin{matrix}S_1&O^{2\times 2}\\O^{2\times 2}&S_1\end{matrix}\right)=\left(\begin{matrix}0&1&0&0\\1&0&0&0\\
0&0&0&1\\0&0&1&0\end{matrix}\right),\quad
M_2=\left(\begin{matrix}S_2&O^{2\times 2}\\O^{2\times 2}&S_2\end{matrix}\right)=\left(\begin{matrix}0&-i&0&0\\i&0&0&0\\
0&0&0&-i\\0&0&i&0\end{matrix}\right),\\
M_3=\left(\begin{matrix}S_3&O^{2\times 2}\\O^{2\times 2}&S_3\end{matrix}\right)=\left(\begin{matrix}1&0&0&0\\0&-1&0&0\\
0&0&1&0\\0&0&0&-1\end{matrix}\right),
\end{multline}
where $S_k$ for $k=1,2,3$ are Pauli matrixes defined as:
\begin{equation}\label{fyfgjguyiuiHHHHHHkjjjjljkjkjknew}
S_1=\left(\begin{matrix}0&1\\1&0\end{matrix}\right),\quad
S_2=\left(\begin{matrix}0&-i\\i&0\end{matrix}\right),\quad
S_3=\left(\begin{matrix}1&0\\0&-1\end{matrix}\right)
\end{equation}
and the sign $\otimes$ in \er{fyfgjguyiuiHHHHHHHHHHuiunew} and
\er{fyfgjguyiuiHHHHHHHHHHuiuhnhgbbnew} means the tensor product of
the matrices, i.e. for given two linear operators
$A:\mathbb{C}^{p}\to\mathbb{C}^{r}$ and
$B:\mathbb{C}^{q}\to\mathbb{C}^{d}$
their tensor product
$A\otimes B$ is a linear operator from
$\mathbb{C}^{p}\otimes\mathbb{C}^{q}$ to
$\mathbb{C}^{r}\otimes\mathbb{C}^{d}$, defined by the identity:
\begin{equation}\label{fyfgjguyiuiHHHHHHkjjjjlnew}
(A\otimes B)\cdot(a\otimes b)=(A\cdot a)\otimes (B\cdot
b)\quad\quad\forall a\in\mathbb{C}^{p},\;\forall b\in\mathbb{C}^{q}.
\end{equation}
Note that, we added an aditional term to the Hamiltonian, namely
$\sum_{j=1}^{n}\frac{\hbar}{4}\vec D_j\cdot\left(curl_{\vec x_j}\vec
v(\vec x_j,t)\,\psi\right)$. In the case of the Newtonian-type
gravity, this term vanishes in every non-rotating and, in
particular, in every inertial coordinate system, however it provides
the invariance of the Dirac equation, under the change of
non-inertial cartesian coordinate system, as can be seen in the
following Theorem \ref{gjghghgghgintintrrZZHHnew}. Thus the
corresponding Dirac equation will be the following:
\begin{multline}\label{vhfffngghkjgghfjjghghghhjghjgghkghgghjhggjjkgfgdSYSPNHHnew}
i\hbar\frac{\partial\psi}{\partial t}=\hat H_0\cdot\psi=\\
\sum_{j=1}^{n}m_jc^2D^j_4\cdot\psi-\sum_{j=1}^{n}c\,\vec
D_j\cdot\left(i\hbar\nabla_{\vec x_j}\psi+\frac{\sigma_j}{c}\vec
A(\vec x_j,t)\psi\right)
%
%
%
+\sum_{j=1}^{n}\sigma_j\left(\Psi(\vec x_j,t)-\frac{1}{c}\vec v(\vec
x_j,t)\cdot\vec A(\vec x_j,t)\right)\psi\\-V\left(\vec
x_1,\ldots,\vec
x_n,t\right)\psi-\sum_{j=1}^{n}\frac{i\hbar}{2}div_{\vec
x_j}\left\{\psi\vec v(\vec
x_j,t)\right\}-\sum_{j=1}^{n}\frac{i\hbar}{2}\nabla_{\vec
x_j}\psi\cdot\vec v(\vec x_j,t)+\sum_{j=1}^{n}\frac{\hbar}{4}\vec
M_j\cdot\left(curl_{\vec x_j}\vec v(\vec x_j,t)\,\psi\right)\\
-\sum_{j=1}^{n}\frac{(g_j-1)\sigma_j\hbar}{2m_jc}\vec
M_j\cdot\left(curl_{\vec x_j}\vec A(\vec x_j,t)\,\psi\right).
\end{multline}
So
%
%
%
%
%
%
\begin{multline}\label{vhfffngghkjgghfjjghghghhjghjgghkghgghjhggjjkgfgdiyfgfjkjgjgSYSPNHHnew}
i\hbar\left(\frac{\partial\psi}{\partial t}+\sum_{j=1}^{n}\vec
v(\vec x_j,t)\cdot\nabla_{\vec
x_j}\psi\right)+\sum_{j=1}^{n}\frac{i\hbar}{2}\left(div_{\vec
x_j}\vec v(\vec x_j,t)\right)\psi=\\
\sum_{j=1}^{n}m_jc^2D^j_4\cdot\psi-\sum_{j=1}^{n}c\,\vec
D_j\cdot\left(i\hbar\nabla_{\vec x_j}\psi+\frac{\sigma_j}{c}\vec
A(\vec x_j,t)\psi\right)
%
%
%
+\sum_{j=1}^{n}\sigma_j\left(\Psi(\vec x_j,t)-\frac{1}{c}\vec v(\vec
x_j,t)\cdot\vec A(\vec x_j,t)\right)\psi\\-V\left(\vec
x_1,\ldots,\vec x_n,t\right)\psi+\sum_{j=1}^{n}\frac{\hbar}{4}\vec
M_j\cdot\left(curl_{\vec x_j}\vec v(\vec
x_j,t)\,\psi\right)\\-\sum_{j=1}^{n}\frac{(g_j-1)\sigma_j\hbar}{2m_jc}\vec
M_j\cdot\left(curl_{\vec x_j}\vec A(\vec x_j,t)\,\psi\right).
\end{multline}
We can rewrite the Dirac equation
\er{vhfffngghkjgghfjjghghghhjghjgghkghgghjhggjjkgfgdiyfgfjkjgjgSYSPNHHnew}
in the following alternative form. First, define the linear
operators $T_0:\mathbb{C}^{4}\to \mathbb{C}^{2}$ and
$T_1:\mathbb{C}^{4}\to \mathbb{C}^{2}$ as:
\begin{equation}\label{fyfgjguyiuiHHHHHHkjjjjljkjkjknewjhkhh}
T_0\cdot (a_1,a_2,b_1,b_2)=(a_1,a_2)\quad\text{and}\quad T_1\cdot
(a_1,a_2,b_1,b_2)=(b_1,b_2)\quad\forall\,(a_1,a_2,b_1,b_2)\in\mathbb{C}^4.
\end{equation}
Then for every $k_1,\ldots ,k_n=0,1$ define $\tilde T_{k_1,\ldots,
k_n}:\mathbb{C}^{4^n}\to \mathbb{C}^{2^n}$ as
\begin{equation}\label{fyfgjguyiuiHHHHHHkjjjjljkjkjknewjhkhhlkjjlljk}
\tilde T_{k_1,\ldots,
k_n}:=\left(T_{k_1}\right)\otimes_{1}\left(T_{k_2}\right)\otimes_{2}\left(T_{k_3}\right)\ldots\otimes_{(n-1)}\left(T_{k_n}\right)\quad\quad\forall\,
k_1,\ldots k_n\in\{0,1\}\,,
\end{equation}
and define $\psi_{k_1,\ldots ,k_n}\in \mathbb{C}^{2^n}$ as:
\begin{equation}\label{fyfgjguyiuiHHHHHHkjjjjljkjkjknewjhkhhlkjjlljkjkjk}
\psi_{k_1,\ldots ,k_n}(\vec x_1,\ldots,\vec x_n,t):=\tilde
T_{k_1,\ldots ,k_n}\cdot\psi(\vec x_1,\ldots,\vec
x_n,t)\quad\quad\forall\, k_1,\ldots k_n\in\{0,1\}\,.
\end{equation}
Then we rewrite
\er{vhfffngghkjgghfjjghghghhjghjgghkghgghjhggjjkgfgdiyfgfjkjgjgSYSPNHHnew}
as:
\begin{multline}\label{vhfffngghkjgghfjjghghghhjghjgghkghgghjhggjjkgfgdiyfgfjkjgjgSYSPNHHnewlgugggg}
i\hbar\left(\frac{\partial}{\partial t}\psi_{k_1,\ldots,
k_n}+\sum_{j=1}^{n}\vec v(\vec x_j,t)\cdot\nabla_{\vec
x_j}\psi_{k_1,\ldots,
k_n}\right)+\sum_{j=1}^{n}\frac{i\hbar}{2}\left(div_{\vec x_j}\vec
v(\vec x_j,t)\right)\psi_{k_1,\ldots,
k_n}\\=\sum_{j=1}^{n}\frac{\hbar}{4}\vec S_j\cdot\left(curl_{\vec
x_j}\vec v(\vec x_j,t)\,\psi_{k_1,\ldots, k_n}\right)+
\sum_{j=1}^{n}m_jc^2\,(-1)^{2+k_j}\,\psi_{k_1,\ldots,
k_n}\\-\sum_{j=1}^{n}c\,\vec S_j\cdot\left(i\hbar\,\nabla_{\vec
x_j}\psi_{k_1,\ldots,k_{(j-1)},(1-k_j),k_{(j+1)},\ldots,
k_n}+\frac{\sigma_j}{c}\vec A(\vec
x_j,t)\,\psi_{k_1,\ldots,k_{(j-1)},(1-k_j),k_{(j+1)},\ldots,
k_n}\right)
%
%
%
\\+\sum_{j=1}^{n}\sigma_j\left(\Psi(\vec x_j,t)-\frac{1}{c}\vec
v(\vec x_j,t)\cdot\vec A(\vec x_j,t)\right)\psi_{k_1,\ldots,
k_n}-V\left(\vec x_1,\ldots,\vec x_n,t\right)\psi_{k_1,\ldots, k_n}\\
-\sum_{j=1}^{n}\frac{(g_j-1)\sigma_j\hbar}{2m_jc}\vec
S_j\cdot\left(curl_{\vec x_j}\vec A(\vec x_j,t)\,\psi_{k_1,\ldots,
k_n}\right)\quad\quad\forall\, k_1,\ldots k_n\in\{0,1\}\,.
\end{multline}
Next, in the similar way as the proof of Theorems
\ref{gjghghgghgintintrrZZjkh} and \ref{gjghghgghgintintrrZZHH} we
can prove the following more general Theorem about the invariance of
the Dirac equation
\er{vhfffngghkjgghfjjghghghhjghjgghkghgghjhggjjkgfgdiyfgfjkjgjgSYSPNHHnew}
under the change of inertial or non-inertial cartesian coordinate
system:
\begin{theorem}\label{gjghghgghgintintrrZZHHnew}
Consider that the change of some cartesian coordinate system $(*)$
to another cartesian coordinate system $(**)$ is given by
\begin{equation}\label{noninchgravortbstrjgghguittu2intrrrZZHHnew}
\begin{cases}
t'=t,\\
\vec x'=A(t)\cdot\vec x+\vec z(t),\\
\vec x'_k=A(t)\cdot\vec x_k+\vec z(t)\quad\quad\forall k=1,\ldots,n,
\end{cases}
\end{equation}
where $A(t)\in SO(3)$ is a rotation. Next, assume that in the
coordinate system $(**)$ we observe a validity of the Dirac equation
of the form:
\begin{multline}\label{MaxVacFull1ninshtrredPPNintrrZZHHnew}
i\hbar\left(\frac{\partial\psi'}{\partial t'}+\sum_{j=1}^{n}\vec
v'(\vec x'_j,t')\cdot\nabla_{\vec
x'_j}\psi'\right)+\sum_{j=1}^{n}\frac{i\hbar}{2}\left(div_{\vec
x'_j}\vec v'(\vec x'_j,t')\right)\psi'=
\\
\sum_{j=1}^{n}m'_jc^2D^j_4\psi'-\sum_{j=1}^{n}c\,\vec
D_j\cdot\left(i\hbar\nabla_{\vec x'_j}\psi'+\frac{\sigma'_j}{c}\vec
A'(\vec x'_j,t')\psi'\right)
%
%
%
+\sum_{j=1}^{n}\sigma'_j\left(\Psi'(\vec x'_j,t')-\frac{1}{c}\vec
v'(\vec x'_j,t')\cdot\vec A'(\vec
x'_j,t')\right)\psi'\\-V'\left(\vec x'_1,\ldots,\vec
x'_n,t'\right)\psi'+\sum_{j=1}^{n}\frac{\hbar}{4}\vec
M_j\cdot\left(curl_{\vec x'_j}\vec v'(\vec
x'_j,t')\,\psi'\right)-\sum_{j=1}^{n}\frac{(g'_j-1)\sigma'_j\hbar}{2m'_jc}\vec
M_j\cdot\left(curl_{\vec x'_j}\vec A'(\vec x'_j,t')\,\psi'\right)
\end{multline}
where $\psi(\vec x_1,\ldots,\vec x_n,t)\in\mathbb{C}^{4^n}$ is a
$4^n$-component complex wave function defined above. Then in the
coordinate system $(*)$ we have the validity of Dirac equation of
the same as \er{MaxVacFull1ninshtrredPPNintrrZZHHnew} form:
\begin{multline}\label{MaxVacFull1ninshtrhjkkredPPNintrrZZHHnew}
i\hbar\left(\frac{\partial\psi}{\partial t}+\sum_{j=1}^{n}\vec
v(\vec x_j,t)\cdot\nabla_{\vec
x_j}\psi\right)+\sum_{j=1}^{n}\frac{i\hbar}{2}\left(div_{\vec
x_j}\vec v(\vec x_j,t)\right)\psi=\\
\sum_{j=1}^{n}m_jc^2D^j_4\psi-\sum_{j=1}^{n}c\,\vec
D_j\cdot\left(i\hbar\nabla_{\vec x_j}\psi+\frac{\sigma_j}{c}\vec
A(\vec x_j,t)\psi\right)
%
%
%
+\sum_{j=1}^{n}\sigma_j\left(\Psi(\vec x_j,t)-\frac{1}{c}\vec v(\vec
x_j,t)\cdot\vec A(\vec x_j,t)\right)\psi\\-V\left(\vec
x_1,\ldots,\vec x_n,t\right)\psi+\sum_{j=1}^{n}\frac{\hbar}{4}\vec
M_j\cdot\left(curl_{\vec x_j}\vec v(\vec
x_j,t)\,\psi\right)-\sum_{j=1}^{n}\frac{(g_j-1)\sigma_j\hbar}{2mc}\vec
M_j\cdot\left(curl_{\vec x_j}\vec A(\vec x_j,t)\,\psi\right),
\end{multline}
provided that we have:
\begin{equation}\label{yuythfgfyftydtydtydtyddyyyhhddhhhredPPN111hgghjgintintrrZZHHnew}
\begin{cases}
V'\left(\vec x'_1,\ldots,\vec x'_n,t'\right)\,=\,V\left(\vec
x_1,\ldots,\vec x_n,t\right),\\
g'_j\,=\,g_j,\\
\sigma'_j\,=\,\sigma_j,\\
m'_j\,=\,m_j,\\
\vec v'(\vec x',t)\,=\,A(t)\cdot \vec v(\vec x,t)+\frac{dA}{dt}(t)\cdot\vec x+\frac{d\vec z}{dt}(t),\\
\vec A'(\vec x',t)\,=\,A(t)\cdot \vec A(\vec x,t),\\
\Psi'(\vec x',t)-\vec v'(\vec x',t)\cdot\vec A'(\vec x',t)\,=\,\Psi(\vec x,t)-\vec v(\vec x,t)\cdot\vec A(\vec x,t),\\
\psi'\left(\vec x'_1,\ldots,\vec
x'_n,t'\right)\,=\,\left(W(t)\otimes_{1}W(t)\otimes_{2}W(t)\ldots\otimes_{(n-1)}W(t)
\right)\cdot\psi\left(\vec x_1,\ldots,\vec x_n,t\right),
\end{cases}
\end{equation}
with $4\times 4$ matrix $W(t)$ defined by
\begin{equation}\label{ghgyfftdtdljkljnnewnnnew}
W(t):=\left(\begin{matrix}U(t)&O^{2\times 2}\\O^{2\times
2}&U(t)\end{matrix}\right)=U(t)\otimes I^{2\times 2},
\end{equation}
where, as before, $U(t)\in SU(2)$
is characterized by:
\begin{equation}\label{gyfyfgfgfgghZZHHnew}
U^*(t)\cdot\vec S\cdot U(t)=A(t)\cdot\vec S.
\end{equation}
where
$\vec S:=(S_1,S_2,S_3)$ with Pauli matrices defined by
\er{fyfgjguyiuiHHHHHHkjjjjljkjkjknew},
that means
\begin{multline*}
\left(U^*(t)\cdot S_1\cdot U(t),U^*(t)\cdot S_2\cdot
U(t),U^*(t)\cdot S_3\cdot U(t)\right)=\\
\left(a_{11}(t)S_1+a_{12}(t)S_2+a_{13}(t)S_3\,,\,a_{21}(t)S_1+a_{22}(t)S_2+a_{23}(t)S_3\,,\,a_{31}(t)S_1+a_{32}(t)S_2+a_{33}(t)S_3\right),
\end{multline*}
where $A(t)=\left\{a_{mk}(t)\right\}_{\{1\leq m,k\leq 3\}}$.
\end{theorem}

Next, in the case that the system of our particles has a positive
energy, define
\begin{equation}\label{gyfyfggjhgjhnewkljjknew}
\phi:=e^{-\frac{1}{\hbar}ic^2\left(\sum_{p=1}^{n}m_p\right)t}\psi\quad\text{and}\quad
\phi_{k_1,\ldots,
k_n}:=e^{-\frac{1}{\hbar}ic^2\left(\sum_{p=1}^{n}m_p\right)t}\psi_{k_1,\ldots,
k_n}\quad\quad\forall\, k_1,\ldots k_n\in\{0,1\}\,.
\end{equation}
Then we rewrite the Dirac equations
\er{vhfffngghkjgghfjjghghghhjghjgghkghgghjhggjjkgfgdiyfgfjkjgjgSYSPNHHnewlgugggg}
as
\begin{multline}\label{vhfffngghkjgghfjjghghghhjghjgghkghgghjhggjjkgfgdiyfgfjkjgjgSYSPNHHnewlguggggjkkhh}
\sum_{j=1}^{n}m_jc^2\,\left(1-(-1)^{2+k_j}\right)\,\phi_{k_1,\ldots,
k_n}\\+i\hbar\left(\frac{\partial}{\partial t}\phi_{k_1,\ldots,
k_n}+\sum_{j=1}^{n}\vec v(\vec x_j,t)\cdot\nabla_{\vec
x_j}\phi_{k_1,\ldots,
k_n}\right)+\sum_{j=1}^{n}\frac{i\hbar}{2}\left(div_{\vec x_j}\vec
v(\vec x_j,t)\right)\phi_{k_1,\ldots,
k_n}\\=\sum_{j=1}^{n}\frac{\hbar}{4}\vec S_j\cdot\left(curl_{\vec
x_j}\vec v(\vec x_j,t)\,\phi_{k_1,\ldots,
k_n}\right)-\sum_{j=1}^{n}\frac{(g_j-1)\sigma_j\hbar}{2m_jc}\vec
S_j\cdot\left(curl_{\vec x_j}\vec A(\vec x_j,t)\,\phi_{k_1,\ldots,
k_n}\right)\\-\sum_{j=1}^{n}c\,\vec
S_j\cdot\left(i\hbar\,\nabla_{\vec
x_j}\phi_{k_1,\ldots,k_{(j-1)},(1-k_j),k_{(j+1)},\ldots,
k_n}+\frac{\sigma_j}{c}\vec A(\vec
x_j,t)\,\phi_{k_1,\ldots,k_{(j-1)},(1-k_j),k_{(j+1)},\ldots,
k_n}\right)
%
%
%
\\+\sum_{j=1}^{n}\sigma_j\left(\Psi(\vec x_j,t)-\frac{1}{c}\vec
v(\vec x_j,t)\cdot\vec A(\vec x_j,t)\right)\phi_{k_1,\ldots,
k_n}-V\left(\vec x_1,\ldots,\vec x_n,t\right)\phi_{k_1,\ldots, k_n}\\
\quad\quad\forall\, k_1,\ldots k_n\in\{0,1\}\,.
\end{multline}
In, particular,
\begin{multline}\label{vhfffngghkjgghfjjghghghhjghjgghkghgghjhggjjkgfgdiyfgfjkjgjgSYSPNHHnewlguggggjkkhhjhjh}
i\hbar\left(\frac{\partial}{\partial t}\phi_{0,\ldots,
0}+\sum_{j=1}^{n}\vec v(\vec x_j,t)\cdot\nabla_{\vec
x_j}\phi_{0,\ldots,
0}\right)+\sum_{j=1}^{n}\frac{i\hbar}{2}\left(div_{\vec x_j}\vec
v(\vec x_j,t)\right)\phi_{0,\ldots,
0}\\=\sum_{j=1}^{n}\frac{\hbar}{4}\vec S_j\cdot\left(curl_{\vec
x_j}\vec v(\vec x_j,t)\,\phi_{0,\ldots,
0}\right)-\sum_{j=1}^{n}\frac{(g_j-1)\sigma_j\hbar}{2m_jc}\vec
S_j\cdot\left(curl_{\vec x_j}\vec A(\vec x_j,t)\,\phi_{0,\ldots,
0}\right)\\-\sum_{j=1}^{n}c\,\vec S_j\cdot\left(i\hbar\,\nabla_{\vec
x_j}\phi_{0_1,\ldots,0_{(j-1)},1_j,0_{(j+1)},\ldots,
0_n}+\frac{\sigma_j}{c}\vec A(\vec
x_j,t)\,\phi_{0_1,\ldots,0_{(j-1)},1_j,0_{(j+1)},\ldots, 0_n}\right)
%
%
%
\\+\sum_{j=1}^{n}\sigma_j\left(\Psi(\vec x_j,t)-\frac{1}{c}\vec
v(\vec x_j,t)\cdot\vec A(\vec x_j,t)\right)\phi_{0,\ldots,
0}-V\left(\vec x_1,\ldots,\vec x_n,t\right)\phi_{0,\ldots, 0}\,.
\end{multline}
On the other hand, if $(k_1+\ldots+k_n)\geq 1$ then
$$\sum_{j=1}^{n}m_j\left(1-(-1)^{2+k_j}\right)\geq 2\min{\{m_1,\ldots,m_n\}}>0.$$ Thus, as before in the proof of
\er{vhfffngghkjgghfjjghghghSYShmyuuiiuuhmhmiopoopnniukjhjkhhjl;jkljhjjkllkjhykkJJ},
for
\er{vhfffngghkjgghfjjghghghhjghjgghkghgghjhggjjkgfgdiyfgfjkjgjgSYSPNHHnewlguggggjkkhh}
where $(k_1+\ldots+k_n)\geq 1$, in the non-relativistic limit we
have,
\begin{multline}\label{vhfffngghkjgghfjjghghghhjghjgghkghgghjhggjjkgfgdiyfgfjkjgjgSYSPNHHnewlguggggjkkhhgghjkjhjhmjkh}
\phi_{k_1,\ldots,
k_n}\approx-\sum_{j=1}^{n}\frac{1}{c\left(\sum_{p=1}^{n}m_p\,\left(1-(-1)^{2+k_p}\right)\right)}\Bigg\{\\
\vec S_j\cdot\left(i\hbar\,\nabla_{\vec
x_j}\phi_{k_1,\ldots,k_{(j-1)},(1-k_j),k_{(j+1)},\ldots,
k_n}+\frac{\sigma_j}{c}\vec A(\vec
x_j,t)\,\phi_{k_1,\ldots,k_{(j-1)},(1-k_j),k_{(j+1)},\ldots,
k_n}\right)\Bigg\}\\ \forall\, k_1,\ldots
k_n\in\{0,1\},\;\;(k_1+\ldots+k_n)\geq 1\,.
\end{multline}
Then, using again the mentioned above non-relativistic
approximation, we further approximate
\er{vhfffngghkjgghfjjghghghhjghjgghkghgghjhggjjkgfgdiyfgfjkjgjgSYSPNHHnewlguggggjkkhhgghjkjhjhmjkh}
as:
\begin{multline}\label{vhfffngghkjgghfjjghghghhjghjgghkghgghjhggjjkgfgdiyfgfjkjgjgSYSPNHHnewlguggggjkkhhgghjkjhjhmjkhgjggugu}
\phi_{0_1,\ldots,0_{(j-1)},1_j,0_{(j+1)},\ldots,
0_n}\approx-\frac{1}{2m_jc}\vec S_j\cdot\left(i\hbar\,\nabla_{\vec
x_j}\phi_{0,\ldots,0}+\frac{\sigma_j}{c}\vec A(\vec
x_j,t)\,\phi_{0,\ldots,0}\right)\quad \forall\, j=1,\ldots n\\
\text{and}\quad\quad\quad\quad \phi_{k_1,\ldots, k_n}\approx
0\quad\text{if}\quad k_1+\ldots +k_n\geq 2.
\end{multline}
Thus inserting
\er{vhfffngghkjgghfjjghghghhjghjgghkghgghjhggjjkgfgdiyfgfjkjgjgSYSPNHHnewlguggggjkkhhgghjkjhjhmjkhgjggugu}
into
\er{vhfffngghkjgghfjjghghghhjghjgghkghgghjhggjjkgfgdiyfgfjkjgjgSYSPNHHnewlguggggjkkhhjhjh}
gives
\begin{multline}\label{vhfffngghkjgghfjjghghghhjghjgghkghgghjhggjjkgfgdiyfgfjkjgjgSYSPNHHnewlguggggjkkhhjhjhgugb}
i\hbar\left(\frac{\partial}{\partial t}\phi_{0,\ldots,
0}+\sum_{j=1}^{n}\vec v(\vec x_j,t)\cdot\nabla_{\vec
x_j}\phi_{0,\ldots,
0}\right)+\sum_{j=1}^{n}\frac{i\hbar}{2}\left(div_{\vec x_j}\vec
v(\vec x_j,t)\right)\phi_{0,\ldots, 0}=\\
\sum_{j=1}^{n}\frac{\hbar}{4}\vec S_j\cdot\left(curl_{\vec x_j}\vec
v(\vec x_j,t)\,\phi_{0,\ldots, 0}\right)
+\sum_{j=1}^{n}\frac{1}{2m_j}\,\vec
S_j\cdot\Bigg(i\hbar\,\nabla_{\vec x_j}\left\{\vec
S_j\cdot\left(i\hbar\,\nabla_{\vec
x_j}\phi_{0,\ldots,0}+\frac{\sigma_j}{c}\vec A(\vec
x_j,t)\,\phi_{0,\ldots,0}\right)\right\}\\+\frac{\sigma_j}{c}\vec
A(\vec x_j,t)\,\left\{\vec S_j\cdot\left(i\hbar\,\nabla_{\vec
x_j}\phi_{0,\ldots,0}+\frac{\sigma_j}{c}\vec A(\vec
x_j,t)\,\phi_{0,\ldots,0}\right)\right\}\Bigg)-\sum_{j=1}^{n}\frac{(g_j-1)\sigma_j\hbar}{2m_jc}\vec
S_j\cdot\left(curl_{\vec x_j}\vec A(\vec x_j,t)\,\phi_{0,\ldots,
0}\right)
%
%
%
\\+\sum_{j=1}^{n}\sigma_j\left(\Psi(\vec x_j,t)-\frac{1}{c}\vec
v(\vec x_j,t)\cdot\vec A(\vec x_j,t)\right)\phi_{0,\ldots,
0}-V\left(\vec x_1,\ldots,\vec x_n,t\right)\phi_{0,\ldots, 0}\,.
\end{multline}
that we rewrite as a non-relativistic Shr\"{o}dinger-Pauli equation,
that we studied above in
\er{vhfffngghkjgghfjjghghghhjghjgghkghgghjhggjjkgfgdiyfgfjkjgjgSYSPNHH}:
\begin{multline}\label{vhfffngghkjgghfjjghghghhjghjgghkghgghjhggjjkgfgdiyfgfjkjgjgSYSPNHHpklk}
i\hbar\left(\frac{\partial}{\partial t}\phi_{0,\ldots,
0}+\sum_{j=1}^{n}\vec v(\vec x_j,t)\cdot\nabla_{\vec
x_j}\phi_{0,\ldots,
0}\right)+\sum_{j=1}^{n}\frac{i\hbar}{2}\left(div_{\vec x_j}\vec
v(\vec x_j,t)\right)\phi_{0,\ldots, 0}\\=
-\sum_{j=1}^{n}\frac{\hbar^2}{2m_j}\Delta_{\vec x_j}\phi_{0,\ldots,
0}-V\left(\vec x_1,\ldots,\vec x_n,t\right)\phi_{0,\ldots,
0}\\+\sum_{j=1}^{n}\frac{ i\hbar\sigma_j}{2m_jc}div_{\vec
x_j}\left\{\phi_{0,\ldots, 0}\vec A(\vec
x_j,t)\right\}+\sum_{j=1}^{n}\frac{ i\hbar\sigma_j}{2m_jc}\vec
A(\vec x_j,t)\cdot\nabla_{\vec x_j}\phi_{0,\ldots,
0}\\+\sum_{j=1}^{n}\left(\sigma_j\Psi(\vec
x_j,t)-\frac{\sigma_j}{c}\vec A(\vec x_j,t)\cdot\vec v(\vec
x_j,t)+\frac{\sigma^2_j}{2m_jc^2}\left|\vec A(\vec
x_j,t)\right|^2\right)\phi_{0,\ldots,
0}\\-\sum_{j=1}^{n}\frac{g_j\sigma_j\hbar}{2m_jc}\vec
S_j\cdot\left(curl_{\vec x_j}\vec A(\vec x_j,t)\,\phi_{0,\ldots,
0}\right)+\sum_{j=1}^{n}\frac{\hbar}{4}\vec S_j\cdot\left(curl_{\vec
x_j}\vec v(\vec x_j,t)\,\phi_{0,\ldots, 0}\right).
\end{multline}

Next, consider the relativistic-like Lagrangian of the motion of
system of $n$ spin-half quantum micro-particles with inertial masses
$m_1,\ldots, m_n$ and the charges $\sigma_1,\ldots,\sigma_n$ in the
outer gravitational and electromagnetical fields with
characteristics $\vec v(\vec x,t)$, $\vec A(\vec x,t)$ and
$\Psi(\vec x,t)$ and additional conservative field with potential
$V(\vec y_1,\ldots,\vec y_n,t)$. Then consider a Lagrangian density
$L$ defined by
\begin{multline}\label{vhfffngghkjgghDDmmkkksmfggffgKKnew}
L\left(\psi,\vec A,\Psi,\vec v,\vec x_1,\ldots,\vec x_n,t\right):=\\
\frac{i\hbar}{2}\left(\left(\frac{\partial\psi}{\partial
t}+\sum\limits_{k=1}^{n}\vec v(\vec x_k,t)\cdot\nabla_{\vec
x_k}\psi\right)\cdot\bar\psi-\psi\cdot\left(\frac{\partial\bar\psi}{\partial
t}+\sum\limits_{k=1}^{n}\vec v(\vec x_k,t)\cdot\nabla_{\vec
x_k}\bar\psi\right)\right)\\
+\sum\limits_{k=1}^{n}\frac{c}{2}\left(\vec
D_k\cdot\left(i\hbar\nabla_{\vec x_k}\psi+\frac{\sigma_k}{c}\vec
A(\vec x_k,t)\psi\right)\right)\cdot\bar\psi-\psi\cdot\left(\vec
D_k\cdot\left(i\hbar\nabla_{\vec x_k}\bar\psi-\frac{\sigma_k}{c}\vec
A(\vec x_k,t)\bar\psi\right)\right)\\
+V\left(\vec x_1,\ldots,\vec x_n,t\right)\psi\cdot\bar\psi
-\sum\limits_{k=1}^{n}m_kc^2\left(D^k_4\cdot\psi\right)\cdot\bar\psi-\sum\limits_{k=1}^{n}\sigma_k\left(\Psi(\vec
x_k,t)-\frac{1}{c}\vec v(\vec x_k,t)\cdot\vec A(\vec
x_k,t)\right)\psi\cdot\bar\psi
\\+\sum_{k=1}^{n}\frac{(g_k-1)\sigma_k\hbar}{2m_kc}\left(\left(\vec
M_k\cdot curl_{\vec x_k}\vec A(\vec
x_k,t)\right)\cdot\psi\right)\cdot\bar\psi-\sum\limits_{k=1}^{n}\frac{\hbar}{4}\left(\left(\vec
M_k\cdot curl_{\vec x_k}\vec v(\vec
x_k,t)\right)\cdot\psi\right)\cdot\bar\psi.
\end{multline}
where $\psi:=\psi(\vec x_1,\ldots\vec x_n,t)\in \mathbb{C}^{4^n}$ is
a wave function of the system.
Then, as before, we can get that the density $L$ is invariant under
the change of inertial or non-inertial cartesian coordinate systems
of the form
\begin{equation*}
\begin{cases}
t'=t\\
\vec x'=A(t)\cdot\vec x+\vec z(t)\\
\vec x'_k=A(t)\cdot\vec x_k+\vec z(t)\quad\quad\forall
k=1,\ldots,n,\end{cases}
\end{equation*}
provided, we take
\er{yuythfgfyftydtydtydtyddyyyhhddhhhredPPN111hgghjgintintrrZZHHnew}
into account. Next we investigate stationary points of the
functional
\begin{equation}\label{btfffygtgyggyDDmmkkksmKKnew}
J(\psi)=\int_0^T\int_{\left(\mathbb{R}^3\right)^n}L\left(\psi,\vec
A,\Psi,\vec v,\vec x_1,\ldots,\vec x_n,t\right)d\vec x_1\ldots,d\vec
x_n dt.
\end{equation}
Then
\begin{multline}\label{vhfffngghkjgghDDmmkkkghgsmKKnew}
0=\frac{\delta L}{\delta (\bar\psi)}=
i\hbar\left(\frac{\partial\psi}{\partial
t}+\sum\limits_{k=1}^{n}\frac{1}{2}\vec v(\vec
x_k,t)\cdot\nabla_{\vec
x_k}\psi+\sum\limits_{k=1}^{n}\frac{1}{2}div_{\vec
x_k}\left\{\psi\vec v(\vec x_k,t)\right\}\right)\\
+\sum\limits_{k=1}^{n}c\vec D_k\cdot\left(i\hbar\nabla_{\vec
x_k}\psi+\frac{\sigma_k}{c}\vec A(\vec
x_k,t)\psi\right)-\sum\limits_{k=1}^{n}m_kc^2D^k_4\cdot\psi
+V\left(\vec x_1,\ldots,\vec x_n,t\right)\psi
\\-\sum\limits_{k=1}^{n}\sigma_k\left(\Psi(\vec x_k,t)-\frac{1}{c}\vec
v(\vec x_k,t)\cdot\vec A(\vec x_k,t)\right)\psi
-\sum\limits_{k=1}^{n}\frac{\hbar}{4}\left(\vec M_k\cdot curl_{\vec
x_k}\vec v(\vec x_k,t)\right)\cdot\psi
\\+\sum_{k=1}^{n}\frac{(g_k-1)\sigma_k\hbar}{2m_kc}\left(\vec
M_k\cdot curl_{\vec x_k}\vec A(\vec x_k,t)\right)\cdot\psi ,
\end{multline}
and
\begin{multline}\label{vhfffngghkjgghDDmmkkkghguhyuysmKKnew}
0=\frac{\delta L}{\delta (\psi)}= \bar
{(i)}\hbar\left(\frac{\partial\bar\psi}{\partial
t}+\sum\limits_{k=1}^{n}\frac{1}{2}\vec v(\vec
x_k,t)\cdot\nabla_{\vec
x_k}\bar\psi+\sum\limits_{k=1}^{n}\frac{1}{2}div_{\vec
x_k}\left\{\bar\psi\vec v(\vec x_k,t)\right\}\right)\\
+\sum\limits_{k=1}^{n}c\vec D_k\cdot\left(\bar
{(i)}\hbar\nabla_{\vec x_k}\bar\psi+\frac{\sigma_k}{c}\vec A(\vec
x_k,t)\bar\psi\right)-\sum\limits_{k=1}^{n}m_kc^2D^k_4\cdot\bar\psi
+V\left(\vec x_1,\ldots,\vec x_n,t\right)\bar\psi
\\-\sum\limits_{k=1}^{n}\sigma_k\left(\Psi(\vec x_k,t)-\frac{1}{c}\vec
v(\vec x_k,t)\cdot\vec A(\vec x_k,t)\right)\bar\psi
-\sum\limits_{k=1}^{n}\frac{\hbar}{4}\left(\vec M_k\cdot curl_{\vec
x_k}\vec v(\vec x_k,t)\right)\cdot\bar\psi\\
+\sum_{k=1}^{n}\frac{(g_k-1)\sigma_k\hbar}{2m_kc}\left(\vec
M^T_k\cdot curl_{\vec x_k}\vec A(\vec x_k,t)\right)\cdot\bar\psi.
\end{multline}
Equation \er{vhfffngghkjgghDDmmkkkghguhyuysmKKnew} is just a complex
conjugate of equation \er{vhfffngghkjgghDDmmkkkghgsmKKnew}. Thus the
Euler-Lagrange for \er{btfffygtgyggyDDmmkkksmKKnew} coincides with
the Dirac equation
\er{vhfffngghkjgghfjjghghghhjghjgghkghgghjhggjjkgfgdiyfgfjkjgjgSYSPNHHnew}.

Finally, assume that the first and the second particles have the
same mass $m_1=m_2$ and the same charge $\sigma_1=\sigma_2$ and
moreover, assume that we have
$$V\left(\vec x_1,\vec x_2,\vec x_3,\ldots,\vec
x_n,t\right)=V\left(\vec x_2,\vec x_1,\vec x_3,\ldots,\vec
x_n,t\right)\quad\quad\forall\,\vec x_1,\vec x_2,\vec
x_3,\ldots,\vec x_n\in\mathbb{R}^3,\quad\forall t.$$ In this case it
can be easily deduced that if $\psi(\vec x_1,\vec x_2,\vec
x_3,\ldots,\vec x_n,t)\in\mathbb{C}^{4^n}$ is a solution of
\er{vhfffngghkjgghfjjghghghhjghjgghkghgghjhggjjkgfgdiyfgfjkjgjgSYSPNHHnew},
then $B_{1,2}\cdot\psi(\vec x_2,\vec x_1,\vec x_3,\ldots,\vec
x_n,t)$ is also a solution of
\er{vhfffngghkjgghfjjghghghhjghjgghkghgghjhggjjkgfgdiyfgfjkjgjgSYSPNHHnew},
where by $B_{1,2}:\mathbb{C}^{4^n}\to \mathbb{C}^{4^n}$ we denote
the linear operator (matrix) defined as:
\begin{equation}\label{fyfgjguyiuiHHgjgjgghhHHkbhbbnew}
B_{1,2}\cdot\left(b_1\otimes b_2\otimes b_3\otimes\ldots\otimes
b_n\right)=\left(b_2\otimes b_1\otimes b_3\otimes\ldots\otimes
b_n\right)\quad\quad\forall\, b_1,\ldots b_n\in\mathbb{C}^{4}.
\end{equation}
Therefore, if $\psi(\vec x_1,\vec x_2,\vec x_3,\ldots,\vec
x_n,t)\in\mathbb{C}^{4^n}$ is a solution of
\er{vhfffngghkjgghfjjghghghhjghjgghkghgghjhggjjkgfgdiyfgfjkjgjgSYSPNHHnew}
then for every $t\geq 0$ we will have
\begin{equation}\label{fyfgjguyiuiHHgjgjgghbmgjgfhfhhHHnew}
B_{1,2}\cdot\psi(\vec x_2,\vec x_1,\vec x_3,\ldots,\vec
x_n,t)=-\psi(\vec x_1,\vec x_2,\vec x_3,\ldots,\vec
x_n,t)\quad\quad\forall\,\vec x_1,\vec x_2,\vec x_3,\ldots,\vec
x_n\in\mathbb{R}^3,
\end{equation}
provided that \er{fyfgjguyiuiHHgjgjgghbmgjgfhfhhHHnew} holds for the
initial instant of time $t=0$. So we have a consistency with the
Pauli Exclusion Principle for two or more identical fermions.

\subsection{Quantum Liouville's equation for a finite system of
spin-half relativistic-like particles arising from the Dirac
equation}\label{hilghjhghfdgdg1133} Consider the
\underline{statistical} description of the motion of a system of $n$
spin-half relativistic-like stable micro-particles with inertial
masses $m_1,\ldots,m_n$ and the charges $\sigma_1,\ldots,\sigma_n$
in the outer gravitational and electromagnetical field with
characteristics $\vec v(\vec x,t)$, $\vec A(\vec x,t)$ and
$\Psi(\vec x,t)$ and additional conservative field with potential
$V(\vec y_1,\ldots,\vec y_n,t)$, taking into account the spin
interaction. Then, as before, the Quantum Liouville's equation for
this system of particles has the following form:
\begin{multline}\label{vhfffngghkjgghfjjghghghhjghjgghkghggkghghjghSYSPNjj1133}
i\hbar\frac{\partial\xi}{\partial t}\left(\vec x_1,\ldots,\vec
x_n,\vec y_1,\ldots,\vec y_n,t\right)=\\ \left(\hat H_0\left(\vec
x_1,\ldots,\vec x_n,t\right)\otimes I\right)\cdot\xi\left(\vec
x_1,\ldots,\vec x_n,\vec y_1,\ldots,\vec
y_n,t\right)\\-\left(I\otimes\hat H^*_0\left(\vec y_1,\ldots,\vec
y_n,t\right)\right)\cdot\xi\left(\vec x_1,\ldots,\vec x_n,\vec
y_1,\ldots,\vec y_n,t\right),
\end{multline}
where $\xi\left(\vec x_1,\ldots,\vec x_n,\vec y_1,\ldots,\vec
y_n,t\right)\in\mathbb{C}^{4^n}\otimes \mathbb{C}^{4^n}$ is a
density-matrix function, $\hat H_0\left(\vec x_1,\ldots,\vec
x_n,t\right)$ is the Hamiltonian operator acting on the variables
$\left(\vec x_1,\ldots,\vec x_n\right)$, $\hat H^*_0\left(\vec
y_1,\ldots,\vec y_n,t\right)$ is the \underline{complex} conjugate
(not the Hermitian adjoint)  to the Hamiltonian operator acting on
the variables $\left(\vec y_1,\ldots,\vec y_n\right)$ and $\hat
H_0\left(\vec x_1,\ldots,\vec x_n,t\right)\otimes I$ and
$I\otimes\hat H^*_0\left(\vec y_1,\ldots,\vec y_n,t\right)$ are
linear operators acting on functions\\ $\xi\left(\vec
x_1,\ldots,\vec x_n,\vec y_1,\ldots,\vec
y_n,t\right)\in\mathbb{C}^{4^n}\otimes \mathbb{C}^{4^n}$, defined by
\begin{multline}\label{fyfgjguyiuiHHgjgjgghbmgjgfhfhhHHjhhhjhhkjkjk33}
\left(\hat H_0\left(\vec x_1,\ldots,\vec x_n,t\right)\otimes
I\right)\cdot\left(\psi_1\left(\vec x_1,\ldots,\vec
x_n,t\right)\otimes \psi_2\left(\vec y_1,\ldots,\vec
y_n,t\right)\right)=\\ \left(\hat H_0\left(\vec x_1,\ldots,\vec
x_n,t\right)\cdot\psi_1\left(\vec x_1,\ldots,\vec
x_n,t\right)\right)\otimes
\psi_2\left(\vec y_1,\ldots,\vec y_n,t\right) \quad\text{and}\quad\\
\left(I\otimes\hat H^*_0\left(\vec y_1,\ldots,\vec
y_n,t\right)\right)\cdot\left(\psi_1\left(\vec x_1,\ldots,\vec
x_n,t\right)\otimes \psi_2\left(\vec y_1,\ldots,\vec
y_n,t\right)\right)=\\ \psi_1\left(\vec x_1,\ldots,\vec
x_n,t\right)\otimes \left(\hat H^*_0\left(\vec y_1,\ldots,\vec
y_n,t\right)\cdot\psi_2\left(\vec y_1,\ldots,\vec
y_n,t\right)\right)\quad\quad\forall\psi_1,\psi_2\in\mathbb{C}^{4^n}.
\end{multline}
Since by
\er{vhfffngghkjgghfjjghghghhjghjgghkghgghjhggjjkgSYSPNHHnew} the
Hamiltonian operator $\hat H_0$ has the forms:
\begin{multline}\label{vhfffngghkjgghfjjghghghhjghjgghkghgghjhggjjkgSYSPNjj1133}
\hat H_0\left(\vec x_1,\ldots,\vec x_n,t\right)\cdot\psi\left(\vec
x_1,\ldots,\vec x_n,t\right)=\\
\sum_{j=1}^{n}m_jc^2D^j_4\cdot\psi-\sum_{j=1}^{n}c\,\vec
D_j\cdot\left(i\hbar\nabla_{\vec x_j}\psi+\frac{\sigma_j}{c}\vec
A(\vec x_j,t)\psi\right)
%
%
%
+\sum_{j=1}^{n}\sigma_j\left(\Psi(\vec x_j,t)-\frac{1}{c}\vec v(\vec
x_j,t)\cdot\vec A(\vec x_j,t)\right)\psi\\-V\left(\vec
x_1,\ldots,\vec
x_n,t\right)\psi-\sum_{j=1}^{n}\frac{i\hbar}{2}div_{\vec
x_j}\left\{\psi\vec v(\vec
x_j,t)\right\}-\sum_{j=1}^{n}\frac{i\hbar}{2}\nabla_{\vec
x_j}\psi\cdot\vec v(\vec x_j,t)+\sum_{j=1}^{n}\frac{\hbar}{4}\vec
M_j\cdot\left(curl_{\vec x_j}\vec v(\vec x_j,t)\,\psi\right)
\\-\sum_{j=1}^{n}\frac{(g_j-1)\sigma_j\hbar}{2m_jc}\vec
M_j\cdot\left(curl_{\vec x_j}\vec A(\vec x_j,t)\,\psi\right),
\end{multline}
and consistently with
\er{vhfffngghkjgghfjjghghghhjghjgghkghgghjhggjjkgSYSPNjj1133} we
write
\begin{multline}\label{vhfffngghkjgghfjjghghghhjghjgghkghgghjhggjjkgSYSPNjjklkl1133}
\hat H^*_0\left(\vec y_1,\ldots,\vec y_n,t\right)\cdot\psi\left(\vec
y_1,\ldots,\vec y_n,t\right)=\\
\sum_{j=1}^{n}m_jc^2D^j_4\cdot\psi-\sum_{j=1}^{n}c\,\vec {\bar
D}_j\cdot\left(-i\hbar\nabla_{\vec y_j}\psi+\frac{\sigma_j}{c}\vec
A(\vec y_j,t)\psi\right)
%
%
%
+\sum_{j=1}^{n}\sigma_j\left(\Psi(\vec y_j,t)-\frac{1}{c}\vec v(\vec
y_j,t)\cdot\vec A(\vec y_j,t)\right)\psi\\-V\left(\vec
y_1,\ldots,\vec
y_n,t\right)\psi+\sum_{j=1}^{n}\frac{i\hbar}{2}div_{\vec
y_j}\left\{\psi\vec v(\vec
y_j,t)\right\}+\sum_{j=1}^{n}\frac{i\hbar}{2}\nabla_{\vec
y_j}\psi\cdot\vec v(\vec y_j,t)+\sum_{j=1}^{n}\frac{\hbar}{4}\vec
{\bar M}_j\cdot\left(curl_{\vec y_j}\vec v(\vec y_j,t)\,\psi\right)
\\-\sum_{j=1}^{n}\frac{(g_j-1)\sigma_j\hbar}{2m_jc}\vec
{\bar M}_j\cdot\left(curl_{\vec y_j}\vec A(\vec y_j,t)\,\psi\right),
\end{multline}
where $\vec {\bar D}_j$ and $\vec {\bar M}_j$ are the the
\underline{complex} conjugate (not the Hermitian adjoint) of the
operators $\vec {D}_j$ and $\vec {M}_j$. Thus we rewrite the
corresponding Quantum Liouville's equation
\er{vhfffngghkjgghfjjghghghhjghjgghkghggkghghjghSYSPNjj1133} as:
%
%
%
%
%
%
\begin{multline}\label{vhfffngghkjgghfjjghghghhjghjgghkghgghjhggjjkgfgdiyfgfjkjgjgSYSPNjj1133}
i\hbar\left(\frac{\partial\xi}{\partial t}+\sum_{j=1}^{n}\vec v(\vec
x_j,t)\cdot\nabla_{\vec x_j}\xi+\sum_{j=1}^{n}\vec v(\vec
y_j,t)\cdot\nabla_{\vec
y_j}\xi\right)+\sum_{j=1}^{n}\frac{i\hbar}{2}\left(div_{\vec
x_j}\vec v(\vec x_j,t)+div_{\vec y_j}\vec v(\vec y_j,t)\right)\xi=\\
\sum_{j=1}^{n}\frac{\hbar}{4}(\vec M_j\otimes
I)\cdot\left(curl_{\vec x_j}\vec v(\vec
x_j,t)\,\xi\right)-\sum_{j=1}^{n}\frac{\hbar}{4} (I\otimes \vec{\bar
M}_j)\cdot\left(curl_{\vec y_j}\vec v(\vec y_j,t)\,\xi\right)\\
+\sum_{j=1}^{n}m_jc^2\left((D^j_4\otimes I)-(I\otimes
D^j_4)\right)\cdot\xi-\sum_{j=1}^{n}c\,(\vec D_j\otimes
I)\cdot\left(i\hbar\nabla_{\vec x_j}\xi+\frac{\sigma_j}{c}\vec
A(\vec x_j,t)\xi\right)\\-\sum_{j=1}^{n}c\,(I\otimes \vec {\bar
D}_j)\cdot\left(i\hbar\nabla_{\vec y_j}\xi-\frac{\sigma_j}{c}\vec
A(\vec y_j,t)\xi\right)
%
-\left(V\left(\vec x_1,\ldots,\vec x_n,t\right)-V\left(\vec
y_1,\ldots,\vec
y_n,t\right)\right)\xi\\+\sum_{j=1}^{n}\left(\sigma_j\Psi(\vec
x_j,t)-\frac{\sigma_j}{c}\vec A(\vec x_j,t)\cdot\vec v(\vec
x_j,t)\right)\xi- \sum_{j=1}^{n}\left(\sigma_j\Psi(\vec
y_j,t)-\frac{\sigma_j}{c}\vec A(\vec y_j,t)\cdot\vec v(\vec
y_j,t)\right)\xi\\-\sum_{j=1}^{n}\frac{(g_j-1)\sigma_j\hbar}{2mc}(\vec
M_j\otimes I)\cdot\left(curl_{\vec x_j}\vec A(\vec
x_j,t)\,\xi\right)
+\sum_{j=1}^{n}\frac{(g_j-1)\sigma_j\hbar}{2mc}(I\otimes \vec{\bar
M}_j)\cdot\left(curl_{\vec y_j}\vec A(\vec y_j,t)\,\xi\right),
\end{multline}
where, as before,
\begin{multline}\label{fyfgjguyiuiHHHHHHkjjjjljhjhhh33}
(D_4^j\otimes I)\cdot(a\otimes b)=(D_4^j\cdot a)\otimes
b\quad\text{and}\quad(I\otimes D_4^j)\cdot(a\otimes b)=a\otimes
(D_4^j\cdot b)\quad\forall a\in\mathbb{C}^{4^n},\;\forall
b\in\mathbb{C}^{4^n},\\
(\vec D_j\otimes I)\cdot(a\otimes b)=(\vec D_j\cdot a)\otimes
b\quad\text{and}\quad(I\otimes\vec D_j)\cdot(a\otimes b)=a\otimes
(\vec D_j\cdot b)\quad\forall a\in\mathbb{C}^{4^n},\;\forall
b\in\mathbb{C}^{4^n},
\\
(\vec M_j\otimes I)\cdot(a\otimes b)=(\vec M_j\cdot a)\otimes
b\quad\text{and}\quad(I\otimes\vec M_j)\cdot(a\otimes b)=a\otimes
(\vec M_j\cdot b)\quad\forall a\in\mathbb{C}^{4^n},\;\forall
b\in\mathbb{C}^{4^n}.
\end{multline}
Next consider a change of some non-inertial cartesian coordinate
system $(*)$ to another cartesian coordinate system $(**)$ of the
form \er{noninchgravortbstrjgghguittu2}:
\begin{equation}\label{noninchgravortbstrghgggSYSPNjj1133}
\begin{cases}
\vec x'=A(t)\cdot\vec x+\vec z(t),\\
t'=t,
\end{cases}
\end{equation}
where $A(t)\in SO(3)$ is a rotation. Then, as before in Theorem
\ref{gjghghgghgintintrrZZHHnew}, we deduce that the Quantum
Liouville's equation equation of the form
\er{vhfffngghkjgghfjjghghghhjghjgghkghgghjhggjjkgfgdiyfgfjkjgjgSYSPNjj1133}
is invariant under the change of non-inertial cartesian coordinate
system, provided we have
\begin{equation}\label{vyfgjhgjhvhgghSYSPNjj1133}
\begin{cases}
\vec x'_j=A(t)\cdot\vec x_j+\vec z(t)\quad\forall j=1,\ldots,n\\
\vec y'_j=A(t)\cdot\vec y_j+\vec z(t)\quad\forall j=1,\ldots,n\\
\xi'\left(\vec x'_1,\ldots,\vec x'_n,\vec y'_1,\ldots,\vec
y'_n,t'\right)=\\
\left(\left(W(t)\otimes_{1}W(t)\ldots\otimes_{(n-1)}W(t)
\right)\otimes\left(\bar W(t)\otimes_{1}\bar
W(t)\ldots\otimes_{(n-1)}\bar W(t) \right)\right)\cdot\xi\left(\vec
x_1,\ldots,\vec x_n,\vec y_1,\ldots,\vec y_n,t\right)
\\
V'\left(\vec y'_1,\ldots,\vec y'_n,t'\right)=V\left(\vec
y_1,\ldots,\vec y_n,t\right)
\\
\vec A'=A(t)\cdot\vec A
\\
\Psi'-\frac{1}{c}\vec A'\cdot\vec v'=\Psi-\frac{1}{c}\vec A\cdot\vec
v,
\end{cases}
\end{equation}
with $4\times 4$ matrix $W(t)$ defined by
\begin{equation}\label{ghgyfftdtdljkljnnewnnnewllluuyu}
W(t):=\left(\begin{matrix}U(t)&O^{2\times 2}\\O^{2\times
2}&U(t)\end{matrix}\right)=U(t)\otimes I^{2\times 2},
\end{equation}
where, as before, $U(t)\in SU(2)$
is characterized by:
\begin{equation}\label{gyfyfgfgfgghZZHHhugg33}
U^*(t)\cdot\vec S\cdot U(t)=A(t)\cdot\vec S,
\end{equation}
where
$\vec S:=(S_1,S_2,S_3)$ with Pauli matrices defined by
\er{fyfgjguyiuiHHHHHHkjjjjljkjkjknew}, and $\bar W(t)$ is the the
\underline{complex} conjugate (not the Hermitian adjoint) of the
matrix $W(t)$.

Next assume that $\hat A\left(\vec x_1,\ldots,\vec x_n,t\right)$ is
some Hermitian operator acting on the functions\\ $\psi\left(\vec
x_1,\ldots,\vec x_n,t\right)\in \mathbb{C}^{4^n}$ with respect to
spatial variables $\left(\vec x_1,\ldots,\vec x_n\right)$. Then the
average $\tilde A(t)$ of $\hat A$ on the density matrix $\xi$ is
defined by:
\begin{equation}\label{gyfyfgfgfgghZZHHhuggkjhjjggaa}
\tilde A(t)=\frac{\int trace\left(\left(\hat A\left(\vec
x_1,\ldots,\vec x_n,t\right)\otimes
I\right)\cdot\xi\right)\left(\vec z_1,\ldots,\vec z_n,\vec
z_1,\ldots,\vec z_n,t\right)d\vec z_1\ldots d\vec z_n}{\int
trace\left(\xi\right)\left(\vec z_1,\ldots,\vec z_n,\vec
z_1,\ldots,\vec z_n,t\right)d\vec z_1\ldots d\vec z_n},
\end{equation}
where for given $R\in \mathbb{C}^{p}\otimes \mathbb{C}^{p}$
$trace(R)$ is a linear functional from $\mathbb{C}^{p}\otimes
\mathbb{C}^{p}$ to $\mathbb{C}$ defined by
\begin{equation}\label{fyfgjguyiuiHHHHHHkjjjjljhjhhhbbbgghgaa}
trace(a\otimes b)=\sum_{k=1}^{p}a_kb_k\quad\forall a=(a_1,\ldots
a_p)\in\mathbb{C}^{p},\;\forall b=(b_1,\ldots b_p)\in\mathbb{C}^{p}.
\end{equation}
On the other hand, using the fact that $\hat A$ is Hermitian, it can
be easily deduced that in addition to
\er{gyfyfgfgfgghZZHHhuggkjhjjggaa} we have the following identity:
\begin{multline}\label{gyfyfgfgfgghZZHHhuggkjhjjggjhhhaa}
\int trace\left(\left(\hat A\left(\vec x_1,\ldots,\vec
x_n,t\right)\otimes I\right)\cdot\xi\right)\left(\vec
z_1,\ldots,\vec z_n,\vec z_1,\ldots,\vec z_n,t\right)d\vec z_1\ldots
d\vec z_n=\\ \int trace\left(\left(I\otimes\hat A^*\left(\vec
y_1,\ldots,\vec y_n,t\right)\right)\cdot\xi\right)\left(\vec
z_1,\ldots,\vec z_n,\vec z_1,\ldots,\vec z_n,t\right)d\vec z_1\ldots
d\vec z_n,
\end{multline}
where $\hat A^*\left(\vec y_1,\ldots,\vec y_n,t\right)$ is the
\underline{complex} conjugate (not the Hermitian adjoint)  to the
operator $\hat A$ acting on the variables $\left(\vec
y_1,\ldots,\vec y_n\right)$ and so,
\begin{equation}\label{gyfyfgfgfgghZZHHhuggkjhjjggjhhh99889aa}
\tilde A(t)=\frac{\int trace\left(\left(I\otimes\hat A^*\left(\vec
y_1,\ldots,\vec y_n,t\right)\right)\cdot\xi\right)\left(\vec
z_1,\ldots,\vec z_n,\vec z_1,\ldots,\vec z_n,t\right)d\vec z_1\ldots
d\vec z_n}{\int trace\left(\xi\right)\left(\vec z_1,\ldots,\vec
z_n,\vec z_1,\ldots,\vec z_n,t\right)d\vec z_1\ldots d\vec z_n}.
\end{equation}
In particular, by inserting \er{gyfyfgfgfgghZZHHhuggkjhjjggjhhhaa}
in the particular case $\hat A=\hat H_0$ into
\er{vhfffngghkjgghfjjghghghhjghjgghkghggkghghjghSYSPNjj1133} we
deduce:
\begin{equation}\label{gyfyfgfgfgghZZHHhuggkjhjjggjhhhgugghaa}
i\hbar\frac{\partial}{\partial t}\int
trace\left(\xi\right)\left(\vec z_1,\ldots,\vec z_n,\vec
z_1,\ldots,\vec z_n,t\right)d\vec z_1\ldots d\vec z_n=0,
\end{equation}
and so,
\begin{equation}\label{gyfyfgfgfgghZZHHhuggkjhjjggjhhhjhggaa}
\int trace\left(\xi\right)\left(\vec z_1,\ldots,\vec z_n,\vec
z_1,\ldots,\vec z_n,t\right)d\vec z_1\ldots d\vec z_n=\int
trace\left(\xi\right)\left(\vec z_1,\ldots,\vec z_n,\vec
z_1,\ldots,\vec z_n,0\right)d\vec z_1\ldots d\vec z_n.
\end{equation}

Next, by
\er{vhfffngghkjgghfjjghghghhjghjgghkghggkghghjghSYSPNjj1133} we have
\begin{multline}\label{vhfffngghkjgghfjjghghghhjghjgghkghggkghghjghSYSPNjjjhhhjjaa}
i\hbar\frac{\partial\xi}{\partial t}\left(\vec x_1,\ldots,\vec
x_n,\vec y_1,\ldots,\vec y_n,t\right)=\\ \left(\hat H_0\left(\vec
x_1,\ldots,\vec x_n,t\right)\otimes I\right)\cdot\xi\left(\vec
x_1,\ldots,\vec x_n,\vec y_1,\ldots,\vec
y_n,t\right)\\-\left(I\otimes\hat H^*_0\left(\vec y_1,\ldots,\vec
y_n,t\right)\right)\cdot\xi\left(\vec x_1,\ldots,\vec x_n,\vec
y_1,\ldots,\vec y_n,t\right).
\end{multline}
Then, denote
\begin{equation}\label{gyfyfgfgfgghZZHHhuggkjhjjggjhhhjhgg22jjjkjkjjaa}
\xi_1\left(\vec x_1,\ldots,\vec x_n,\vec y_1,\ldots,\vec
y_n,t\right):=P\cdot\bar\xi\left(\vec y_1,\ldots,\vec y_n,\vec
x_1,\ldots,\vec x_n,t\right),
\end{equation}
where $P$ is a linear operator (matrix), acting on vectors in
$\mathbb{C}^{4^n}\otimes \mathbb{C}^{4^n}$ and satisfying:
\begin{equation}\label{fyfgjguyiuiHHHHHHkjjjjljhjhhhjhhjjhaa}
P\cdot(a\otimes b)=b\otimes a\quad\forall
a\in\mathbb{C}^{4^n},\;\forall b\in\mathbb{C}^{4^n},
\end{equation}
and consider
\begin{equation}\label{gyfyfgfgfgghZZHHhuggkjhjjggjhhhjhgg22jjjkjkkljjjkkjjjaa}
\zeta\left(\vec x_1,\ldots,\vec x_n,\vec y_1,\ldots,\vec
y_n,t\right):=\xi\left(\vec x_1,\ldots,\vec x_n,\vec y_1,\ldots,\vec
y_n,t\right)-\xi_1\left(\vec x_1,\ldots,\vec x_n,\vec
y_1,\ldots,\vec y_n,t\right).
\end{equation}
Thus
\begin{multline}\label{vhfffngghkjgghfjjghghghhjghjgghkghggkghghjghSYSPNjjjhhhgggjjaa}
-i\hbar P\cdot \frac{\partial\bar\xi}{\partial t}\left(\vec
y_1,\ldots,\vec y_n,\vec x_1,\ldots,\vec x_n,t\right)=\\
\left(I\otimes\hat H^*_0\left(\vec y_1,\ldots,\vec
y_n,t\right)\right)\cdot\left(P\cdot\bar\xi\left(\vec
y_1,\ldots,\vec y_n,\vec x_1,\ldots,\vec
x_n,t\right)\right)\\-\left(\hat H_0\left(\vec x_1,\ldots,\vec
x_n,t\right)\otimes I\right)\cdot\left(P\cdot\bar\xi\left(\vec
y_1,\ldots,\vec y_n,\vec x_1,\ldots,\vec x_n,t\right)\right),
\end{multline}
and then
\begin{multline}\label{vhfffngghkjgghfjjghghghhjghjgghkghggkghghjghSYSPNjjjhhhhjhjjhjkhjhhjjaa88}
i\hbar\frac{\partial\xi_1}{\partial t}\left(\vec x_1,\ldots,\vec
x_n,\vec y_1,\ldots,\vec y_n,t\right)=\\ \left(\hat H_0\left(\vec
x_1,\ldots,\vec x_n,t\right)\otimes I\right)\cdot\xi_1\left(\vec
x_1,\ldots,\vec x_n,\vec y_1,\ldots,\vec
y_n,t\right)\\-\left(I\otimes\hat H^*_0\left(\vec y_1,\ldots,\vec
y_n,t\right)\right)\cdot\xi_1\left(\vec x_1,\ldots,\vec x_n,\vec
y_1,\ldots,\vec y_n,t\right).
\end{multline}
Therefore,
\begin{multline}\label{vhfffngghkjgghfjjghghghhjghjgghkghggkghghjghSYSPNjjjhhhhjhjjhjkhjhhjjaa}
i\hbar\frac{\partial\zeta}{\partial t}\left(\vec x_1,\ldots,\vec
x_n,\vec y_1,\ldots,\vec y_n,t\right)=\\ \left(\hat H_0\left(\vec
x_1,\ldots,\vec x_n,t\right)\otimes I\right)\cdot\zeta\left(\vec
x_1,\ldots,\vec x_n,\vec y_1,\ldots,\vec
y_n,t\right)\\-\left(I\otimes\hat H^*_0\left(\vec y_1,\ldots,\vec
y_n,t\right)\right)\cdot\zeta\left(\vec x_1,\ldots,\vec x_n,\vec
y_1,\ldots,\vec y_n,t\right).
\end{multline}
Thus if $\zeta$ satisfies initial condition $\zeta\left(\vec
x_1,\ldots,\vec x_n,\vec y_1,\ldots,\vec y_n,0\right)=0$ then
$\zeta\left(\vec x_1,\ldots,\vec x_n,\vec y_1,\ldots,\vec
y_n,t\right)=0$ for every $t>0$. So
\begin{multline}\label{gyfyfgfgfgghZZHHhuggkjhjjggjhhhjhgg22jjjkjkkljjjkkjopiiiojjaa}
\xi\left(\vec x_1,\ldots,\vec x_n,\vec y_1,\ldots,\vec
y_n,0\right)=P\cdot\bar\xi\left(\vec y_1,\ldots,\vec y_n,\vec
x_1,\ldots,\vec x_n,0\right)\quad\text{implies}\quad\\
\xi\left(\vec x_1,\ldots,\vec x_n,\vec y_1,\ldots,\vec
y_n,t\right)=P\cdot\bar\xi\left(\vec y_1,\ldots,\vec y_n,\vec
x_1,\ldots,\vec x_n,t\right)\quad\forall\, t\geq 0.
\end{multline}

Finally, assume that the first and the second particles have the
same mass $m_1=m_2$ and the same charge $\sigma_1=\sigma_2$ and
moreover, assume that we have
$$V\left(\vec x_1,\vec x_2,\vec x_3,\ldots,\vec
x_n,t\right)=V\left(\vec x_2,\vec x_1,\vec x_3,\ldots,\vec
x_n,t\right)\quad\quad\forall\,\vec x_1,\vec x_2,\vec
x_3,\ldots,\vec x_n\in\mathbb{R}^3,\quad\forall t.$$ In this case it
can be easily deduced that if  $\xi\left(\vec x_1,\vec x_2,\vec
x_3,\ldots,\vec x_n,\vec y_1,\vec y_2,\vec y_3,\ldots,\vec
y_n,t\right)\in\mathbb{C}^{4^n}\otimes \mathbb{C}^{4^n}$ is a
solution of
\er{vhfffngghkjgghfjjghghghhjghjgghkghggkghghjghSYSPNjjjhhhjjaa},
then $C_{1,2}\cdot\xi(\vec x_2,\vec x_1,\vec x_3,\ldots,\vec
x_n,\vec y_2,\vec y_1,\vec y_3,\ldots,\vec y_n,t)$
is also a solution of
\er{vhfffngghkjgghfjjghghghhjghjgghkghggkghghjghSYSPNjjjhhhjjaa},
where by $C_{1,2}:\mathbb{C}^{4^n}\otimes \mathbb{C}^{4^n}\to
\mathbb{C}^{4^n}\otimes \mathbb{C}^{4^n}$
we denote the linear operator (matrix) defined as:
\begin{multline}\label{fyfgjguyiuiHHgjgjgghhHHjjhjhaa}
C_{1,2}\cdot\left(\left(a_1\otimes a_2\otimes
a_3\otimes\ldots\otimes a_n\right)\otimes \left(b_1\otimes
b_2\otimes b_3\otimes\ldots\otimes b_n\right)\right)=\\
\left(\left(a_2\otimes a_1\otimes a_3\otimes\ldots\otimes
a_n\right)\otimes \left(b_2\otimes b_1\otimes
b_3\otimes\ldots\otimes b_n\right)\right)\quad\forall\, a_1,\ldots
a_n,b_1,\ldots,b_n\in\mathbb{C}^{4}.
\end{multline}
Therefore, if $\xi\left(\vec x_1,\vec x_2,\vec x_3,\ldots,\vec
x_n,\vec y_1,\vec y_2,\vec y_3,\ldots,\vec
y_n,t\right)\in\mathbb{C}^{4^n}\otimes \mathbb{C}^{4^n}$ is a
solution of
\er{vhfffngghkjgghfjjghghghhjghjgghkghggkghghjghSYSPNjjjhhhjjaa}
then for every $t\geq 0$ we will have
\begin{multline}\label{fyfgjguyiuiHHgjgjgghbmgjgfhfhhHHrraa}
C_{1,2}\cdot\xi(\vec x_2,\vec x_1,\vec x_3,\ldots,\vec x_n,\vec
y_2,\vec y_1,\vec y_3,\ldots,\vec y_n,t)=\xi\left(\vec x_1,\vec
x_2,\vec x_3,\ldots,\vec x_n,\vec y_1,\vec y_2,\vec y_3,\ldots,\vec
y_n,t\right)
\\ \quad\quad\forall\,\vec x_1,\vec x_2,\vec
x_3,\ldots,\vec x_n,\vec y_1,\vec y_2,\vec y_3,\ldots,\vec
y_n\in\mathbb{R}^3,
\end{multline}
provided that \er{fyfgjguyiuiHHgjgjgghbmgjgfhfhhHHrraa} holds for
the initial instant of time $t=0$. So we have a consistency with the
Pauli Exclusion Principle for two or more identical fermions.

Next given $\xi\left(\vec x_1,\ldots,\vec x_n,\vec y_1,\ldots,\vec
y_n,t\right)$ that satisfies
\begin{equation}\label{gyfyfgfgfgghZZHHhuggkjhjjggjhhhjhgg22jjjkjkkljjjkkjopiiiojjkljkjkjjaa}
\xi\left(\vec x_1,\ldots,\vec x_n,\vec y_1,\ldots,\vec
y_n,t\right)=P\cdot\bar\xi\left(\vec y_1,\ldots,\vec y_n,\vec
x_1,\ldots,\vec x_n,t\right)\quad\forall\, t\geq 0,
\end{equation}
define the operator $\hat R_\xi$ by
\begin{multline}\label{gyfyfgfgfgghZZHHhuggkjhjjggjhhhjhgg22jjjkjkkljjjkkjopiiiojjkljkjkjkljkjkjjaa}
\hat R_\xi\left(\vec x_1,\ldots,\vec x_n,t\right)\cdot\psi\left(\vec
x_1,\ldots,\vec x_n,t\right)=\\ \int B\left(\xi\left(\vec
x_1,\ldots,\vec x_n,\vec y_1,\ldots,\vec y_n,t\right),\psi\left(\vec
y_1,\ldots,\vec y_n,t\right)\right)d\vec y_1\ldots d\vec y_n,
\end{multline}
where $B$ is a bilinear mapping, acting on the pairs of vectors in
$\left(\mathbb{C}^{4^n}\otimes \mathbb{C}^{4^n}\right)\times
\mathbb{C}^{4^n}$, taking the values in $\mathbb{C}^{4^n}$ and
having the form
\begin{equation}\label{fyfgjguyiuiHHHHHHkjjjjljhjhhhjhhjjhyhyujkjkaa}
B\left((a\otimes b), c\right)=(b\cdot c)a\quad\forall
a\in\mathbb{C}^{4^n},\;\forall b\in\mathbb{C}^{4^n},\;\forall
c\in\mathbb{C}^{4^n}.
\end{equation}
Then the mapping $\xi\to\hat R_\xi$ is one-to-one. Moreover, by
\er{gyfyfgfgfgghZZHHhuggkjhjjggjhhhjhgg22jjjkjkkljjjkkjopiiiojjkljkjkjjaa}
$\hat R_\xi$ is an Hermitian operator. Finally by
\er{vhfffngghkjgghfjjghghghhjghjgghkghggkghghjghSYSPNjjjhhhjjaa} we
have
\begin{multline}\label{vhfffngghkjgghfjjghghghhjghjgghkghggkghghjghSYSPNjjjhhhjjkjkjklkjljljjaa}
i\hbar\frac{\partial\hat R_\xi}{\partial t}\left(\vec
x_1,\ldots,\vec x_n,t\right)=\\ \hat H_0\left(\vec x_1,\ldots,\vec
x_n,t\right)\cdot\hat R_\xi\left(\vec x_1,\ldots,\vec
x_n,t\right)-\hat R_\xi\left(\vec x_1,\ldots,\vec
x_n,t\right)\cdot\hat H_0\left(\vec x_1,\ldots,\vec x_n,t\right).
\end{multline}
Equation
\er{vhfffngghkjgghfjjghghghhjghjgghkghggkghghjghSYSPNjjjhhhjjkjkjklkjljljjaa}
is equivalent to
\er{vhfffngghkjgghfjjghghghhjghjgghkghggkghghjghSYSPNjjjhhhjjaa}.

Next, clearly,  $\xi$ satisfies
\er{fyfgjguyiuiHHgjgjgghbmgjgfhfhhHHrraa} if and only if we have the
following relation:
\begin{multline}\label{fyfgjguyiuiHHgjgjgghbmgjgfhfhhHHljjjhjhjhyyuyuaa}
B_{1,2}\cdot\psi(\vec x_2,\vec x_1,\vec x_3,\ldots,\vec
x_n,t)=-\psi(\vec x_1,\vec x_2,\vec x_3,\ldots,\vec
x_n,t)\;\;\forall\,\vec x_1,\ldots,\vec
x_n\in\mathbb{R}^3\quad\text{and}\quad \phi=\hat R_\xi\cdot\psi\\
\quad\text{implies}\quad B_{1,2}\cdot\phi(\vec x_2,\vec x_1,\vec
x_3,\ldots,\vec x_n,t)=-\phi(\vec x_1,\vec x_2,\vec x_3,\ldots,\vec
x_n,t)\;\;\forall\,\vec x_1,\ldots,\vec x_n\in\mathbb{R}^3,
\end{multline}
where by $B_{1,2}:\mathbb{C}^{4^n}\to \mathbb{C}^{4^n}$ we denote
the linear operator (matrix) defined as in
\er{fyfgjguyiuiHHgjgjgghhHHkbhbbnew} by the following:
\begin{equation}\label{fyfgjguyiuiHHgjgjgghhHHhuhuhuhaa}
B_{1,2}\cdot\left(a_1\otimes a_2\otimes a_3\otimes\ldots\otimes
a_n\right)=\left(a_2\otimes a_1\otimes a_3\otimes\ldots\otimes
a_n\right)\quad\quad\forall\, a_1,\ldots a_n\in\mathbb{C}^{4}.
\end{equation}

Finally, assume that $\hat A\left(\vec x_1,\ldots,\vec x_n,t\right)$
is some Hermitian operator acting on the functions\\ $\psi\left(\vec
x_1,\ldots,\vec x_n,t\right)\in \mathbb{C}^{4^n}$ with respect to
spatial variables $\left(\vec x_1,\ldots,\vec x_n\right)$. Then, we
remind that the average $\tilde A(t)$ of $\hat A$ on the density
matrix $\xi$ is defined by \er{gyfyfgfgfgghZZHHhuggkjhjjggaa}:
\begin{equation}\label{gyfyfgfgfgghZZHHhuggkjhjjggkhjhjhaa}
\tilde A(t)=\frac{\int trace\left(\left(\hat A\left(\vec
x_1,\ldots,\vec x_n,t\right)\otimes
I\right)\cdot\xi\right)\left(\vec z_1,\ldots,\vec z_n,\vec
z_1,\ldots,\vec z_n,t\right)d\vec z_1\ldots d\vec z_n}{\int
trace\left(\xi\right)\left(\vec z_1,\ldots,\vec z_n,\vec
z_1,\ldots,\vec z_n,t\right)d\vec z_1\ldots d\vec z_n}.
\end{equation}
On the other hand,
\begin{equation}\label{gyfyfgfgfgghZZHHhuggkjhjjgg22hjhjjhjhmnbmbhjhjaa}
trace \left(\hat A\circ \hat R_\xi\right)=\int trace\left(\left(\hat
A\left(\vec x_1,\ldots,\vec x_n,t\right)\otimes
I\right)\cdot\xi\right)\left(\vec z_1,\ldots,\vec z_n,\vec
z_1,\ldots,\vec z_n,t\right)d\vec z_1\ldots d\vec z_n,
\end{equation}
and so,
\begin{equation}\label{gyfyfgfgfgghZZHHhuggkjhjjgg22hjhjjhjhmnbmbaa}
\tilde A(t)=\frac{trace \left(\hat A\circ \hat R_\xi\right)}{trace
\left(\hat R_\xi\right)}=\frac{trace \left( \hat R_\xi\circ \hat
A\right)}{trace \left(\hat R_\xi\right)},
\end{equation}
where by the trace we mean the trace of an operator on a Hilbert
space. Moreover, by \er{gyfyfgfgfgghZZHHhuggkjhjjggjhhhjhggaa} and
\er{gyfyfgfgfgghZZHHhuggkjhjjgg22hjhjjhjhmnbmbhjhjaa} we have
\begin{equation}\label{gyfyfgfgfgghZZHHhuggkjhjjgg22hjhjjhjhmnbmbjkjhjhjhhaa}
trace \left( \hat R_\xi \right)(t)=trace \left( \hat
R_\xi\right)(0).
\end{equation}

\subsubsection{Thermodynamical equilibrium for relativistic-like spin-half particles}
Let $\xi\left(\vec x_1,\ldots,\vec x_n,\vec y_1,\ldots,\vec
y_n,t\right)$ be the density matrix of the quantum system, ruled by
the Hamiltonian operator $\hat H_0$ from
\er{vhfffngghkjgghfjjghghghhjghjgghkghgghjhggjjkgSYSPNjj1133}, in
the case of \underline{thermodynamical equilibrium} and $\hat R_\xi$
is given by
\er{gyfyfgfgfgghZZHHhuggkjhjjggjhhhjhgg22jjjkjkkljjjkkjopiiiojjkljkjkjkljkjkjjaa}
for this $\xi$. Then in the case that $\frac{\partial \hat
H_0}{\partial t}\equiv 0$ and the given equilibrium system rests
macroscopically, i.e. it has the macroscopical velocity field zero:
$\vec u(\vec x,t)\equiv 0$ it is well known that
\begin{equation}\label{vhfffngghkjgghfjjghghghSYSPNjljllkljjhkhkhklpl;lpoopkklkljlrrccnnzz}
\hat R_\xi=\frac{1}{trace \left( f\circ \hat H_0 \right)} f\circ
\hat H_0,
\end{equation}
where $f(z):\mathbb{C}\to \mathbb{C}$ is some holomorphic function
defined as a sum of the power series
\begin{equation}\label{vhfffngghkjgghfjjghghghSYSPNjljllkljjhkhkhklpl;lpoopkklkljlrrcchhkhccnnzz}
f(z):=\sum_{m=0}^{+\infty}a_mz^m,
\end{equation}
with $a_m\in\mathbb{C}$, and the operator $f\circ \hat H_{0}$ is
given by
\begin{equation}\label{vhfffngghkjgghfjjghghghSYSPNjljllkljjhkhkhklpl;lpoopkklkljlrrcchhkhcchjhjgccnnzz}
f\circ \hat H_{0}:=\sum_{m=0}^{+\infty}a_m\hat H^m_{0}.
\end{equation}
We remind that in the case of a \underline{non-relativistic system}
in thermostat, having the Kelvin temperature $T$, we considered the
canonical ensemble:
\begin{equation}\label{vhfffngghkjgghfjjghghghhjghjgghkghggkghghjghSYSPNjjjhhhjjkjkjklkjljljkkhhkjiijjihhjjhjhjhjhkkhhk11cchjhjhcckk}
f(s)=e^{-\frac{s}{kT}}=\sum_{m=0}^{+\infty}\frac{1}{m!}\left(-\frac{s}{kT}\right)^m\quad\forall
s\in\mathbb{C},
\end{equation}
where $k$ is the Boltzmann constant and by $f\circ \hat H_0$ we
denoted the following operator
\begin{equation}\label{vhfffngghkjgghfjjghghghhjghjgghkghggkghghjghSYSPNjjjhhhjjkjkjklkjljljkkhhkjiijjihhjjhjhjhjhkkhhknnhkhh11cchjgjgjcckk}
f\circ \hat
H_0=\sum_{m=0}^{+\infty}\frac{1}{m!}\left(-\frac{1}{kT}\hat
H_0\right)^m\quad\forall s\in\mathbb{C}.
\end{equation}
However, in the \underline{non-relativistic case} the Hamiltonian
operator $\hat H_0$ was bounded from below and then the trace of
$f\circ \hat H_0$ was finite. On the other hand, in the case of the
Dirac system, the infinitely small negative energies are possible
and then the trace of $f\circ \hat H_0$ could be infinite. Thus, we
need to assume in addition that $f(s)$ decays rapidly as
$s\to\pm\infty$. For example, the possible choice is the following:
\begin{equation}\label{vhfffngghkjgghfjjghghghhjghjgghkghggkghghjghSYSPNjjjhhhjjkjkjklkjljljkkhhkjiijjihhjjhjhjhjhkkhhk11cchjhjhcckkjkljklkhjjh}
f(s)=e^{-\frac{s^2}{2\left(\sum_{j=1}^{n}m_j\right)c^2kT}}=\sum_{m=0}^{+\infty}\frac{1}{m!}\left(-\frac{s^2}{2\left(\sum_{j=1}^{n}m_j\right)c^2kT}\right)^m\quad\forall
s\in\mathbb{C},
\end{equation}
where $k$ is the Boltzmann constant and by $f\circ \hat H_0$ we
denote the following operator
\begin{equation}\label{vhfffngghkjgghfjjghghghhjghjgghkghggkghghjghSYSPNjjjhhhjjkjkjklkjljljkkhhkjiijjihhjjhjhjhjhkkhhknnhkhh11cchjgjgjcckkjjkjkjk}
f\circ \hat
H_0=\sum_{m=0}^{+\infty}\frac{1}{m!}\left(-\frac{1}{2\left(\sum_{j=1}^{n}m_j\right)c^2kT}\hat
H^2_0\right)^m\quad\forall s\in\mathbb{C}.
\end{equation}
We would like to find alternative form of the above law of
thermodynamical equilibrium, which is invariant under the change of
inertial or non-inertial cartesian coordinate system. Then it is
clear that if the given system rests in the old coordinate system,
then it obviously has non-trivial macroscopical velocity field $\vec
u(\vec x,t)$ in the new one. Moreover, $\vec u(\vec x,t)$ can depend
on $\vec x$ and $t$, as it indeed happens in the case of a rotation
of the new coordinate system with respect to the old one. On the
other hand the concept of thermodynamical equilibrium is clearly
independent of the coordinate system.

In order to find the forms of the law of thermodynamical
equilibrium, which are indeed invariant under the change of inertial
or non-inertial cartesian coordinate system we follow the steps as
below. Given a speed-like vector field $\vec u=\vec u(\vec x,t)$,
define
\begin{multline}\label{vhfffngghkjgghfjjghghghSYSPNjljll'ljjjjjjhkjhjhjjkjkkjj88}
\hat H_{\vec u}\left(\vec x_1,\ldots,\vec x_n,t\right):=\\
\hat H_0\left(\vec x_1,\ldots,\vec
x_n,t\right)-\frac{1}{2}\left(\sum_{j=1}^{n}\hat {\vec P}_j\cdot\vec
u(\vec x_j,t)+\vec u(\vec x_j,t)\cdot \hat {\vec
P}_j\right)-\sum_{j=1}^{n}\frac{\hbar}{4}\vec
M_j\cdot\left(curl_{\vec x_j}\vec u(\vec x_j,t)\right),
\end{multline}
where $\hat{\vec P}_j$ is operator defined as
\begin{equation}\label{vhfffngghkjgghfjjghghghSYSPNjljll'ljjjjjjhkjhjhjjkjkkljjljjkkjj88}
\hat {\vec P}_j\cdot\psi\left(\vec x_1,\ldots,\vec
x_n\right)=-i\hbar\nabla_{\vec x_j}\psi\left(\vec x_1,\ldots,\vec
x_n\right).
\end{equation}
On the other hand, by
\er{vhfffngghkjgghfjjghghghhjghjgghkghgghjhggjjkgSYSPNjj1133} the
Hamiltonian operator $\hat H_0$ has the forms:
\begin{multline}\label{vhfffngghkjgghfjjghghghhjghjgghkghgghjhggjjkgSYSPNjjjljjkljjkjkjj88kk88}
\hat H_0\left(\vec x_1,\ldots,\vec x_n,t\right)\cdot\psi\left(\vec
x_1,\ldots,\vec x_n,t\right)=\\
\sum_{j=1}^{n}m_jc^2D^j_4\cdot\psi-\sum_{j=1}^{n}c\,\vec
D_j\cdot\left(i\hbar\nabla_{\vec x_j}\psi+\frac{\sigma_j}{c}\vec
A(\vec x_j,t)\psi\right)
%
%
%
+\sum_{j=1}^{n}\sigma_j\left(\Psi(\vec x_j,t)-\frac{1}{c}\vec v(\vec
x_j,t)\cdot\vec A(\vec x_j,t)\right)\psi\\-V\left(\vec
x_1,\ldots,\vec
x_n,t\right)\psi-\sum_{j=1}^{n}\frac{i\hbar}{2}div_{\vec
x_j}\left\{\psi\vec v(\vec
x_j,t)\right\}-\sum_{j=1}^{n}\frac{i\hbar}{2}\nabla_{\vec
x_j}\psi\cdot\vec v(\vec x_j,t)+\sum_{j=1}^{n}\frac{\hbar}{4}\vec
M_j\cdot\left(curl_{\vec x_j}\vec v(\vec x_j,t)\,\psi\right)
\\-\sum_{j=1}^{n}\frac{(g_j-1)\sigma_j\hbar}{2m_jc}\vec
M_j\cdot\left(curl_{\vec x_j}\vec A(\vec x_j,t)\,\psi\right).
\end{multline}
Then by
\er{vhfffngghkjgghfjjghghghhjghjgghkghgghjhggjjkgSYSPNjjjljjkljjkjkjj88kk88}
and \er{vhfffngghkjgghfjjghghghSYSPNjljll'ljjjjjjhkjhjhjjkjkkjj88}
we have
\begin{multline}\label{vhfffngghkjgghfjjghghghhjghjgghkghgghjhggjjkgSYSPNjjjljjkljjkjkjj88}
\hat H_{\vec u}\left(\vec x_1,\ldots,\vec
x_n,t\right)\cdot\psi\left(\vec x_1,\ldots,\vec x_n,t\right)=
\sum_{j=1}^{n}m_jc^2D^j_4\cdot\psi-\sum_{j=1}^{n}c\,\vec
D_j\cdot\left(i\hbar\nabla_{\vec x_j}\psi+\frac{\sigma_j}{c}\vec
A(\vec x_j,t)\psi\right)
%
%
%
\\+\sum_{j=1}^{n}\sigma_j\left(\Psi(\vec x_j,t)-\frac{1}{c}\vec
v(\vec x_j,t)\cdot\vec A(\vec x_j,t)\right)\psi-V\left(\vec
x_1,\ldots,\vec
x_n,t\right)\psi\\-\sum_{j=1}^{n}\frac{i\hbar}{2}div_{\vec
x_j}\left\{\psi\left(\vec v(\vec x_j,t)-\vec u(\vec
x_j,t)\right)\right\}-\sum_{j=1}^{n}\frac{i\hbar}{2}\nabla_{\vec
x_j}\psi\cdot\left(\vec v(\vec x_j,t)-\vec u(\vec
x_j,t)\right)\\+\sum_{j=1}^{n}\frac{\hbar}{4}\vec
M_j\cdot\left(curl_{\vec x_j}\left(\vec v(\vec x_j,t)-\vec u(\vec
x_j,t)\right)\,\psi\right)
-\sum_{j=1}^{n}\frac{(g_j-1)\sigma_j\hbar}{2m_jc}\vec
M_j\cdot\left(curl_{\vec x_j}\vec A(\vec x_j,t)\,\psi\right).
\end{multline}
Thus, as before, we can prove that
\er{vhfffngghkjgghfjjghghghhjghjgghkghgghjhggjjkgSYSPNjjjljjkljjkjkjj88}
is invariant under the change of inertial or non-inertial cartesian
coordinate systems i.e.
\begin{multline}\label{vhfffngghkjgghfjjghghghhjghjgghkghgghjhggjjkgSYSPNjjjljjkljjkjkjjz1188}
\hat H'_{\vec u'}\left(\vec x'_1,\ldots,\vec
x'_n,t'\right)\cdot\psi'\left(\vec x'_1,\ldots,\vec
x'_n,t'\right)=\phi'\left(\vec x'_1,\ldots,\vec x'_n,t'\right)
\quad\text{if and only if}\quad\\  \hat H_{\vec u}\left(\vec
x_1,\ldots,\vec x_n,t\right)\cdot\psi\left(\vec x_1,\ldots,\vec
x_n,t\right)=\phi\left(\vec x_1,\ldots,\vec x_n,t\right),
\end{multline}
provided that, as before in
\er{yuythfgfyftydtydtydtyddyyyhhddhhhredPPN111hgghjgintintrrZZHHnew},
we have
\begin{equation}\label{yuythfgfyftydtydtydtyddyyyhhddhhhredPPN111hgghjgintintrrZZHHhkkhjhj1188}
\begin{cases}
V'\left(\vec x'_1,\ldots,\vec x'_n,t'\right)\,=\,V\left(\vec
x_1,\ldots,\vec x_n,t\right),\\
\sigma'_j\,=\,\sigma_j,\\
m'_j\,=\,m_j,\\
g'_j\,=\,g_j,\\
\vec v'(\vec x',t)\,=\,A(t)\cdot \vec v(\vec x,t)+\frac{dA}{dt}(t)\cdot\vec x+\frac{d\vec z}{dt}(t),\\
\vec u'(\vec x',t)\,=\,A(t)\cdot \vec u(\vec x,t)+\frac{dA}{dt}(t)\cdot\vec x+\frac{d\vec z}{dt}(t),\\
\vec A'(\vec x',t)\,=\,A(t)\cdot \vec A(\vec x,t),\\
\Psi'(\vec x',t)-\vec v'(\vec x',t)\cdot\vec A'(\vec x',t)\,=\,\Psi(\vec x,t)-\vec v(\vec x,t)\cdot\vec A(\vec x,t),\\
\psi'\left(\vec x'_1,\ldots,\vec
x'_n,t'\right)\,=\,\left(W(t)\otimes_{1}W(t)\otimes_{2}W(t)\ldots\otimes_{(n-1)}W(t)
\right)\cdot\psi\left(\vec x_1,\ldots,\vec x_n,t\right),\\
\phi'\left(\vec x'_1,\ldots,\vec
x'_n,t'\right)\,=\,\left(W(t)\otimes_{1}W(t)\otimes_{2}W(t)\ldots\otimes_{(n-1)}W(t)
\right)\cdot\phi\left(\vec x_1,\ldots,\vec x_n,t\right).
\end{cases}
\end{equation}
%
%
%
%
%
%

Next consider
\begin{equation}\label{vhfffngghkjgghfjjghghghSYSPNjljllkljjhkhkhklpl;lpooprrccnnzz}
\hat R_\xi=\frac{1}{trace \left( f\circ \hat H_{\vec u} \right)}
f\circ \hat H_{\vec u},
\end{equation}
where $f(z):\mathbb{C}\to \mathbb{C}$ is some holomorphic function
defined as a sum of the power series
\begin{equation}\label{vhfffngghkjgghfjjghghghSYSPNjljllkljjhkhkhklpl;lpoopkklkljlrrcchhkhcchhhjjkljccnnzz}
f(z):=\sum_{m=0}^{+\infty}a_mz^m,
\end{equation}
with $a_m\in\mathbb{C}$, such that $f(s)$ decays rapidly as
$s\to\pm\infty$. For example, we can take $f(s)$ be defined by
\er{vhfffngghkjgghfjjghghghhjghjgghkghggkghghjghSYSPNjjjhhhjjkjkjklkjljljkkhhkjiijjihhjjhjhjhjhkkhhk11cchjhjhcckkjkljklkhjjh}.
Moreover, consider the operator $f\circ \hat H_{\vec u}$, given by:
\begin{equation}\label{vhfffngghkjgghfjjghghghSYSPNjljllkljjhkhkhklpl;lpoopkklkljlrrcchhkhcchjhjgcckjhhjccnnzz}
f\circ \hat H_{\vec u}:=\sum_{m=0}^{+\infty}a_m\hat H^m_{\vec u}.
\end{equation}
%
%
%
%
%
%
%
%
We would like to note here that by
\er{vhfffngghkjgghfjjghghghhjghjgghkghgghjhggjjkgSYSPNjjjljjkljjkjkjj88}
in the case that vector field $\vec u(\vec x,t)$ satisfies
\begin{equation}\label{btfffygtgyggyijhhkkSYSPNuhuygyygyggyuyyuyuykuhghgjjojiyyyuyuuyyijyyuyuhghghguydd}
\left|\vec u(\vec x,t)-\vec v(\vec
x,t)\right|\,<\,c\quad\quad\quad\quad\forall\, (\vec x,t),
\end{equation}
since $f(s)$ decays rapidly as $s\to\pm\infty$, there exists a
density matrix $\xi_1$, such that the operator $\hat R_{\xi_1}$,
given by
\er{gyfyfgfgfgghZZHHhuggkjhjjggjhhhjhgg22jjjkjkkljjjkkjopiiiojjkljkjkjkljkjkjjaa},
equals to the operator $f\circ \hat H_{\vec u}$ and thus
\er{vhfffngghkjgghfjjghghghSYSPNjljllkljjhkhkhklpl;lpooprrccnnzz}
indeed has sense. Next by
\er{vhfffngghkjgghfjjghghghhjghjgghkghgghjhggjjkgSYSPNjjjljjkljjkjkjjz1188}
and
\er{vhfffngghkjgghfjjghghghSYSPNjljllkljjhkhkhklpl;lpoopkklkljlrrcchhkhcchjhjgcckjhhjccnnzz}
the operator $\hat R_\xi$ in the left hand side of
\er{vhfffngghkjgghfjjghghghSYSPNjljllkljjhkhkhklpl;lpooprrccnnzz} is
invariant under the change of inertial or non-inertial cartesian
coordinate systems i.e.
\begin{multline}\label{vhfffngghkjgghfjjghghghhjghjgghkghgghjhggjjkgSYSPNjjjljjkljjkjkjjz1155kjjkjknnzz}
\hat R_{\xi'}\left(\vec x'_1,\ldots,\vec
x'_n,t'\right)\cdot\psi'\left(\vec x'_1,\ldots,\vec
x'_n,t'\right)=\phi'\left(\vec x'_1,\ldots,\vec x'_n,t'\right)
\quad\text{if and only if}\quad\\  \hat R_\xi\left(\vec
x_1,\ldots,\vec x_n,t\right)\cdot\psi\left(\vec x_1,\ldots,\vec
x_n,t\right)=\phi\left(\vec x_1,\ldots,\vec x_n,t\right),
\end{multline}
provided that, as before, we have
\er{yuythfgfyftydtydtydtyddyyyhhddhhhredPPN111hgghjgintintrrZZHHhkkhjhj1188}
and
\begin{multline}\label{vyfgjhgjhvhgghSYSPNjj11jhjggjjhjgzz}
\xi'\left(\vec x'_1,\ldots,\vec x'_n,\vec y'_1,\ldots,\vec
y'_n,t'\right)=\\
\left(\left(W(t)\otimes_{1}W(t)\ldots\otimes_{(n-1)}W(t)
\right)\otimes\left(\bar W(t)\otimes_{1}\bar
W(t)\ldots\otimes_{(n-1)}\bar W(t) \right)\right)\cdot\xi\left(\vec
x_1,\ldots,\vec x_n,\vec y_1,\ldots,\vec y_n,t\right).
\end{multline}
Moreover, in the case $\vec u\equiv 0$
\er{vhfffngghkjgghfjjghghghSYSPNjljllkljjhkhkhklpl;lpooprrccnnzz}
coincides with
\er{vhfffngghkjgghfjjghghghSYSPNjljllkljjhkhkhklpl;lpoopkklkljlrrccnnzz}.
However, we still need to derive the restrictions on the field $\vec
u$ and the Hamiltonian operator $\hat H_0$, providing that our
system can indeed be found in the state of thermodynamical
equilibrium. We remind that in the case $\vec u\equiv 0$ the
appropriate restriction is $\frac{\partial \hat H_0}{\partial
t}\equiv 0$. In order to get these restrictions in the general case,
we need to insert $\hat R_\xi$ in
\er{vhfffngghkjgghfjjghghghSYSPNjljllkljjhkhkhklpl;lpooprrccnnzz}
into the equation in
\er{vhfffngghkjgghfjjghghghhjghjgghkghggkghghjghSYSPNjjjhhhjjkjkjklkjljljjaa}
which is equivalent to the quantum Liouville equation.

Therefore, assume that the vector field $\vec u$ satisfies satisfies
\er{vhfffngghkjgghfjjghghghhjghjgghkghggSYSPNkljljjkklmjkjkjjk,mmm,kl;k;lklkjkjkkjkjkjkjljljjljklggghgkhhkhkhkkklkklklkljkhkhjjkhjh22}
with \er{vhfffngghkjgghfjjghghghSYSPNjljll'lhkhhjgggjgjgghgjkjjk22}
i.e
\begin{equation}\label{vhfffngghkjgghfjjghghghhjghjgghkghggSYSPNkljljjkklmjkjkjjk,mmm,kl;k;lklkjkjkkjkjkjkjljljjljklggghgkhhkhkhkkklkklklkljkhkhjjkhjhmnnmmn1188}
\begin{cases}
\frac{\partial U}{\partial t}\left(\vec x_1,\ldots,\vec
x_n,t\right)+\sum_{j=1}^{n}\vec u( \vec x_j,t)\cdot\nabla_{\vec
x_j}U\left(\vec x_1,\ldots,\vec x_n,t\right)=0,\\
\frac{\partial}{\partial t}\left(\vec u(\vec x,t)-\vec v(\vec
x,t)\right)-curl_{\vec x}\left(\vec u(\vec x,t)\times\left(\vec
u(\vec x,t)-\vec v(\vec x,t)\right)\right)+\left(\Div_{\vec
x}\left(\vec u(\vec x,t)-\vec v(\vec x,t)\right)\right)\vec
u(\vec x,t)=0,\\
\frac{\partial\vec A}{\partial t}(\vec x,t)-curl_{\vec x}\left(\vec
u(\vec x,t)\times\vec A(\vec x,t)\right)+\left(\Div_{\vec x}\vec
A(\vec x,t)\right)\vec u(\vec x,t)=0,
\\
d_{\vec x}\vec u(\vec x,t)+\left\{d_{\vec x}\vec u(\vec
x,t)\right\}^T=0,
\end{cases}
\end{equation}
where
\begin{multline}\label{vhfffngghkjgghfjjghghghSYSPNjljll'lhkhhjgggjgjgghgjkjjkjkjk1188}
U\left(\vec x_1,\ldots,\vec x_n,t\right):=\\
\sum_{j=1}^{n}\left(\frac{\sigma^2_j}{2m_jc^2}\left|\vec A(\vec
x_j,t)\right|^2+\sigma_j\Psi(\vec x_j,t)-\frac{\sigma_j}{c}\vec
A(\vec x_j,t)\cdot\vec v(\vec x_j,t)\right)-V\left(\vec
x_1,\ldots,\vec x_n,t\right).
\end{multline}
Then, by Proposition \ref{eqv3} below we have
\begin{equation}\label{vhfffngghkjgghfjjghghghhjghjgghkghggkghghjghSYSPNjjjhhhjjkjkjklkjljljkkhhkjiijjijhbmbgghghcchkhhnnzz}
i\hbar\frac{\partial}{\partial t}\left(f\circ \hat H_{\vec
u}\right)= \hat H_0\cdot\left(f\circ \hat H_{\vec
u}\right)-\left(f\circ \hat H_{\vec u}\right)\cdot\hat H_0,
\end{equation}
Then by
\er{vhfffngghkjgghfjjghghghhjghjgghkghggkghghjghSYSPNjjjhhhjjkjkjklkjljljkkhhkjiijjijhbmbgghghcchkhhnnzz}
together with
\er{gyfyfgfgfgghZZHHhuggkjhjjgg22hjhjjhjhmnbmbjkjhjhjhhaa} we deduce
that
\begin{equation}\label{vhfffngghkjgghfjjghghghhjghjgghkghggkghghjghSYSPNjjjhhhjjkjkjklkjljljkkhhkjiijjijhbmbgghghcchkhhjkjkknnzz}
i\hbar\frac{\partial}{\partial t}\left(\hat R_\xi\right)= \hat
H_0\cdot\left(\hat R_\xi\right)-\left(\hat R_\xi\right)\cdot\hat
H_0,
\end{equation}
where $\hat R_\xi$ is the operator in the left hand side of
\er{vhfffngghkjgghfjjghghghSYSPNjljllkljjhkhkhklpl;lpooprrccnnzz}.
So we indeed get
\er{vhfffngghkjgghfjjghghghhjghjgghkghggkghghjghSYSPNjjjhhhjjkjkjklkjljljjaa}
in the case of
\er{vhfffngghkjgghfjjghghghhjghjgghkghggSYSPNkljljjkklmjkjkjjk,mmm,kl;k;lklkjkjkkjkjkjkjljljjljklggghgkhhkhkhkkklkklklkljkhkhjjkhjhmnnmmn1188},
\er{vhfffngghkjgghfjjghghghSYSPNjljll'lhkhhjgggjgjgghgjkjjkjkjk1188}.
Moreover,
\er{vhfffngghkjgghfjjghghghhjghjgghkghggSYSPNkljljjkklmjkjkjjk,mmm,kl;k;lklkjkjkkjkjkjkjljljjljklggghgkhhkhkhkkklkklklkljkhkhjjkhjhmnnmmn1188},
\er{vhfffngghkjgghfjjghghghSYSPNjljll'lhkhhjgggjgjgghgjkjjkjkjk1188}
and
\er{btfffygtgyggyijhhkkSYSPNuhuygyygyggyuyyuyuykuhghgjjojiyyyuyuuyyijyyuyuhghghguydd}
are invariant under the change of inertial or non-inertial
coordinate system.

\begin{proposition}\label{eqv3}
Assume that the speed-like vector field $\vec u$ satisfies
\er{vhfffngghkjgghfjjghghghhjghjgghkghggSYSPNkljljjkklmjkjkjjk,mmm,kl;k;lklkjkjkkjkjkjkjljljjljklggghgkhhkhkhkkklkklklkljkhkhjjkhjhmnnmmn1188}
and
\er{vhfffngghkjgghfjjghghghSYSPNjljll'lhkhhjgggjgjgghgjkjjkjkjk1188}.
Next assume that the holomorphic function $f(z):\mathbb{C}\to
\mathbb{C}$ defined as a sum of the power series
\begin{equation}\label{vhfffngghkjgghfjjghghghhjghjgghkghggkghghjghSYSPNjjjhhhjjkjkjklkjljljkkhhkjiijjihhjjhjhjhjhkkhhk11cchjhjhcc}
f(z):=\sum_{m=0}^{+\infty}a_mz^m,
\end{equation}
with $a_m\in\mathbb{C}$, is such that for the operator $f\circ \hat
H_{\vec u}$, given by:
\begin{equation}\label{vhfffngghkjgghfjjghghghhjghjgghkghggkghghjghSYSPNjjjhhhjjkjkjklkjljljkkhhkjiijjihhjjhjhjhjhkkhhknnhkhh11cchjgjgjcc}
f\circ \hat H_{\vec u}:=\sum_{m=0}^{+\infty}a_m\hat H^m_{\vec u},
\end{equation}
where the operator $\hat H_{\vec u}$  is given by
\er{vhfffngghkjgghfjjghghghhjghjgghkghgghjhggjjkgSYSPNjjjljjkljjkjkjj88},
there exists a density matrix $\xi\left(\vec x_1,\ldots,\vec
x_n,\vec y_1,\ldots,\vec y_n,t\right)$, such that
\begin{equation}\label{vhfffngghkjgghfjjghghghhjghjgghkghggkghghjghSYSPNjjjhhhjjkjkjklkjljljkkhhkjiijjijhbmbgghgh11ggghhgfccjggjgjhccjhjgcc}
\left(f\circ\hat H_{\vec u}\right)\left(\vec x_1,\ldots,\vec
x_n,t\right)=\hat R_{\xi}\left(\vec x_1,\ldots,\vec x_n,t\right),
\end{equation}
where $\hat R_\xi$ is given by
\er{gyfyfgfgfgghZZHHhuggkjhjjggjhhhjhgg22jjjkjkkljjjkkjopiiiojjkljkjkjkljkjkjjaa}.
Then the operator $f\circ \hat H_{\vec u}$ satisfies:
\begin{equation}\label{vhfffngghkjgghfjjghghghhjghjgghkghggkghghjghSYSPNjjjhhhjjkjkjklkjljljkkhhkjiijjijhbmbgghghihihhccnnnmvcc}
i\hbar\frac{\partial}{\partial t}\left(f\circ \hat H_{\vec
u}\right)= \hat H_0\cdot\left(f\circ \hat H_{\vec
u}\right)-\left(f\circ \hat H_{\vec u}\right)\cdot\hat H_0,
\end{equation}
or equivalently
\begin{multline}\label{vhfffngghkjgghfjjghghghhjghjgghkghggkghghjghSYSPNjjjhhhjjccnnncc}
i\hbar\frac{\partial\xi}{\partial t}\left(\vec x_1,\ldots,\vec
x_n,\vec y_1,\ldots,\vec y_n,t\right)=\\ \left(\hat H_0\left(\vec
x_1,\ldots,\vec x_n,t\right)\otimes I\right)\cdot\xi\left(\vec
x_1,\ldots,\vec x_n,\vec y_1,\ldots,\vec
y_n,t\right)\\-\left(I\otimes\hat H^*_0\left(\vec y_1,\ldots,\vec
y_n,t\right)\right)\cdot\xi\left(\vec x_1,\ldots,\vec x_n,\vec
y_1,\ldots,\vec y_n,t\right).
\end{multline}
Moreover, $f\circ \hat H_{\vec u}$ is invariant under the change of
inertial or non-inertial cartesian coordinate systems i.e.
\begin{multline}\label{vhfffngghkjgghfjjghghghhjghjgghkghgghjhggjjkgSYSPNjjjljjkljjkjkjjzjjjkkhhklhhlcciyiyyyyccmmnvv}
\left(f\circ\hat H'_{\vec u'}\right)\left(\vec x'_1,\ldots,\vec
x'_n,t'\right)\cdot\psi'\left(\vec x'_1,\ldots,\vec
x'_n,t'\right)=\phi'\left(\vec x'_1,\ldots,\vec x'_n,t'\right)
\quad\text{if and only if}\quad\\  \left(f\circ\hat H_{\vec
u}\right)\left(\vec x_1,\ldots,\vec x_n,t\right)\cdot\psi\left(\vec
x_1,\ldots,\vec x_n,t\right)=\phi\left(\vec x_1,\ldots,\vec
x_n,t\right),
\end{multline}
provided that we have
\er{yuythfgfyftydtydtydtyddyyyhhddhhhredPPN111hgghjgintintrrZZHHhkkhjhj1188}.
Finally, if we assume that the first and the second particles have
the same mass $m_1=m_2$ and the same charge $\sigma_1=\sigma_2$ and
moreover, if we assume that
$$V\left(\vec x_1,\vec x_2,\vec x_3,\ldots,\vec
x_n,t\right)=V\left(\vec x_2,\vec x_1,\vec x_3,\ldots,\vec
x_n,t\right)\quad\quad\forall\,\vec x_1,\vec x_2,\vec
x_3,\ldots,\vec x_n\in\mathbb{R}^3,\quad\forall t,$$ then
\begin{multline}\label{fyfgjguyiuiHHgjgjgghbmgjgfhfhhHHljjjhjhjhyyuyuhhhljjjkjj88cc}
B_{1,2}\cdot\psi(\vec x_2,\vec x_1,\vec x_3,\ldots,\vec
x_n,t)=-\psi(\vec x_1,\vec x_2,\vec x_3,\ldots,\vec
x_n,t)\;\;\forall\,\vec x_1,\ldots,\vec
x_n\in\mathbb{R}^3\quad\text{and}\quad \phi=\left(f\circ\hat H_{\vec u}\right)\cdot\psi\\
\quad\text{implies}\quad B_{1,2}\cdot\phi(\vec x_2,\vec x_1,\vec
x_3,\ldots,\vec x_n,t)=-\phi(\vec x_1,\vec x_2,\vec x_3,\ldots,\vec
x_n,t)\;\;\forall\,\vec x_1,\ldots,\vec x_n\in\mathbb{R}^3,
\end{multline}
where, as before, by $B_{1,2}:\mathbb{C}^{4^n}\to \mathbb{C}^{4^n}$
we denote the linear operator (matrix) defined as in
\er{fyfgjguyiuiHHgjgjgghhHHkbhbbnew} by the following:
\begin{equation}\label{fyfgjguyiuiHHgjgjgghhHHhuhuhuhhkh88cc}
B_{1,2}\cdot\left(a_1\otimes a_2\otimes a_3\otimes\ldots\otimes
a_n\right)=\left(a_2\otimes a_1\otimes a_3\otimes\ldots\otimes
a_n\right)\quad\quad\forall\, a_1,\ldots a_n\in\mathbb{C}^{4}.
\end{equation}
\end{proposition}
\begin{proof}

Again by
\er{vhfffngghkjgghfjjghghghhjghjgghkghggSYSPNkljljjkklmjkjkjjk,mmm,kl;k;lklkjkjkkjkjkjkjljljjljklggghgkhhkhkhkkklkklklkljkhkhjjkhjhmnnmmn1188}
and Proposition \ref{jhjkhhjhjhgjgjjkjk} there exists another
cartesian coordinate system $(**)$ such that under the change of
coordinate system $(*)$ to another cartesian coordinate system
$(**)$, given by \er{noninchgravortbstrjgghguittu2hhhhh}, we have
\begin{equation}
\label{jljjjjkhhhjkkww1188}A(t)\cdot \vec u(\vec x,t)+
\frac{d A}{dt}(t)\cdot\vec x+
\frac{d\vec z}{dt}(t)=\vec u'(\vec x',t')=0.
\end{equation}
Thus, since
\er{vhfffngghkjgghfjjghghghhjghjgghkghggSYSPNkljljjkklmjkjkjjk,mmm,kl;k;lklkjkjkkjkjkjkjljljjljklggghgkhhkhkhkkklkklklkljkhkhjjkhjhmnnmmn1188},
\er{vhfffngghkjgghfjjghghghSYSPNjljll'lhkhhjgggjgjgghgjkjjkjkjk1188}
are invariant under the change of inertial or non-inertial cartesian
coordinate systems, as before, in system $(**)$ we have
\begin{equation}\label{vhfffngghkjgghfjjghghghhjghjgghkghggSYSPNkljljjkklmjkjkjjk,mmm,kl;k;lklkjkjkkjkjkjkjljljjljklggghgkhhkhkhkkklkklklkljkhkhjjkhjhmjmm,mjhjkjkhjhj1188}
\begin{cases}
\frac{\partial U'}{\partial t'}\left(\vec x'_1,\ldots,\vec
x'_n,t'\right)=0,\\
\frac{\partial\vec v'}{\partial t'}(\vec
x',t')=0,\\
\frac{\partial\vec A'}{\partial t'}(\vec x',t')=0,
\\
\vec u'(\vec x',t')=0,
\end{cases}
\end{equation}
where
\begin{multline}\label{vhfffngghkjgghfjjghghghSYSPNjljll'lhkhhjgggjgjgghgjkjjkkjjhj1188}
U'\left(\vec x'_1,\ldots,\vec x'_n,t'\right):=\\
\sum_{j=1}^{n}\left(\frac{(\sigma'_j)^2}{2m'_jc^2}\left|\vec A'(\vec
x'_j,t')\right|^2+\sigma'_j\Psi'(\vec
x'_j,t')-\frac{\sigma'_j}{c}\vec A'(\vec x'_j,t')\cdot\vec v'(\vec
x'_j,t')\right)-V'\left(\vec x'_1,\ldots,\vec x'_n,t'\right).
\end{multline}
On the other hand,
\er{vhfffngghkjgghfjjghghghhjghjgghkghggSYSPNkljljjkklmjkjkjjk,mmm,kl;k;lklkjkjkkjkjkjkjljljjljklggghgkhhkhkhkkklkklklkljkhkhjjkhjhmjmm,mjhjkjkhjhj1188}
is equivalent to
\begin{equation}\label{vhfffnhhjhjjkkjhkjkhkh1188}
\frac{\partial \hat H'_0}{\partial t'}\left(\vec x'_1,\ldots,\vec
x'_n,t'\right)=0\quad\text{and}\quad\vec u'(\vec x',t')=0,
\end{equation}
in system $(**)$. Thus, by \er{vhfffnhhjhjjkkjhkjkhkh1188}, in
system $(**)$ we deduce:
\begin{equation}\label{vhfffngghkjgghfjjghghghhjghjgghkghggkghghjghSYSPNjjjhhhjjkjkjklkjljljkkhhkjjjkljjbbb1188}
\hat H'_0\cdot\hat H'_{\vec u'}-\hat H'_{\vec u'}\cdot\hat H'_0=\hat
H'_0\cdot\hat H'_0-\hat H'_0\cdot\hat
H'_0=0=i\hbar\frac{\partial\hat H'_0}{\partial
t'}=i\hbar\frac{\partial\hat H'_{\vec u'}}{\partial t'}.
\end{equation}
So we get:
\begin{equation}\label{vhfffngghkjgghfjjghghghhjghjgghkghggkghghjghSYSPNjjjhhhjjkjkjklkjljljkkhhkjjjkljjbbbjkljljl1188}
i\hbar\frac{\partial\hat H'_{\vec u'}}{\partial t'}=\hat
H'_0\cdot\hat H'_{\vec u'}-\hat H'_{\vec u'}\cdot\hat H'_0.
\end{equation}
Therefore, given the holomorphic function $f(z):\mathbb{C}\to
\mathbb{C}$ defined as a sum of the power series
\begin{equation}\label{vhfffngghkjgghfjjghghghhjghjgghkghggkghghjghSYSPNjjjhhhjjkjkjklkjljljkkhhkjiijjihhjjhjhjhjhkkhhk1188}
f(z):=\sum_{m=0}^{+\infty}a_mz^m,
\end{equation}
with $a_m\in\mathbb{C}$, if we define the operator $f\circ \hat
H_{\vec u}$ as:
\begin{equation}\label{vhfffngghkjgghfjjghghghhjghjgghkghggkghghjghSYSPNjjjhhhjjkjkjklkjljljkkhhkjiijjihhjjhjhjhjhkkhhknnhkhh1188}
f\circ \hat H_{\vec u}:=\sum_{m=0}^{+\infty}a_m\hat H^m_{\vec u},
\end{equation}
by
\er{vhfffngghkjgghfjjghghghhjghjgghkghggkghghjghSYSPNjjjhhhjjkjkjklkjljljkkhhkjjjkljjbbbjkljljl1188}
and Proposition \ref{hugyfyfffgh} we deduce
\begin{equation}\label{vhfffngghkjgghfjjghghghhjghjgghkghggkghghjghSYSPNjjjhhhjjkjkjklkjljljkkhhkjiijjijhbmbgghghdd1188}
i\hbar\frac{\partial}{\partial t'}\left(f\circ \hat H'_{\vec
u'}\right)= \hat H'_0\cdot\left(f\circ \hat H'_{\vec
u'}\right)-\left(f\circ \hat H'_{\vec u'}\right)\cdot\hat H'_0.
\end{equation}
Moreover, by
\er{vhfffngghkjgghfjjghghghhjghjgghkghgghjhggjjkgSYSPNjjjljjkljjkjkjjz1188}
and
\er{vhfffngghkjgghfjjghghghhjghjgghkghggkghghjghSYSPNjjjhhhjjkjkjklkjljljkkhhkjiijjihhjjhjhjhjhkkhhknnhkhh1188},
we can easily prove that $f\circ \hat H_{\vec u}$ is invariant under
the change of inertial or non-inertial cartesian coordinate systems
i.e.
\begin{multline}\label{vhfffngghkjgghfjjghghghhjghjgghkghgghjhggjjkgSYSPNjjjljjkljjkjkjjzjjjkkhhklhhl1188}
\left(f\circ\hat H'_{\vec u'}\right)\left(\vec x'_1,\ldots,\vec
x'_n,t'\right)\cdot\psi'\left(\vec x'_1,\ldots,\vec
x'_n,t'\right)=\phi'\left(\vec x'_1,\ldots,\vec x'_n,t'\right)
\quad\text{if and only if}\quad\\  \left(f\circ\hat H_{\vec
u}\right)\left(\vec x_1,\ldots,\vec x_n,t\right)\cdot\psi\left(\vec
x_1,\ldots,\vec x_n,t\right)=\phi\left(\vec x_1,\ldots,\vec
x_n,t\right),
\end{multline}
provided that, as before, we have
\er{yuythfgfyftydtydtydtyddyyyhhddhhhredPPN111hgghjgintintrrZZHHhkkhjhj1188}.

Next assume that the holomorphic function $f$ in
\er{vhfffngghkjgghfjjghghghhjghjgghkghggkghghjghSYSPNjjjhhhjjkjkjklkjljljkkhhkjiijjihhjjhjhjhjhkkhhk1188}
is such that there exists a density matrix $\xi\left(\vec
x_1,\ldots,\vec x_n,\vec y_1,\ldots,\vec y_n,t\right)$ satisfying
\begin{equation}\label{vhfffngghkjgghfjjghghghhjghjgghkghggkghghjghSYSPNjjjhhhjjkjkjklkjljljkkhhkjiijjijhbmbgghgh11ggghhgf11gg1188}
\left(f\circ\hat H_{\vec u}\right)\left(\vec x_1,\ldots,\vec
x_n,t\right)=\hat R_{\xi}\left(\vec x_1,\ldots,\vec x_n,t\right),
\end{equation}
where $\hat R_\xi$ is given by
\er{gyfyfgfgfgghZZHHhuggkjhjjggjhhhjhgg22jjjkjkkljjjkkjopiiiojjkljkjkjkljkjkjjaa}.
Then, by
\er{vhfffngghkjgghfjjghghghhjghjgghkghgghjhggjjkgSYSPNjjjljjkljjkjkjjzjjjkkhhklhhl1188}
for the density matrix $\xi'\left(\vec x'_1,\ldots,\vec x'_n,\vec
y'_1,\ldots,\vec y'_n,t'\right)$ we have
\begin{equation}\label{vhfffngghkjgghfjjghghghhjghjgghkghggkghghjghSYSPNjjjhhhjjkjkjklkjljljkkhhkjiijjijhbmbgghgh11ggghhgf1188}
\left(f\circ\hat H'_{\vec u'}\right)\left(\vec x'_1,\ldots,\vec
x'_n,t'\right)=\hat R_{\xi'}\left(\vec x'_1,\ldots,\vec
x'_n,t'\right),
\end{equation}
provided that
\begin{multline}\label{vyfgjhgjhvhgghSYSPNjj11jhjggj88}
\xi'\left(\vec x'_1,\ldots,\vec x'_n,\vec y'_1,\ldots,\vec
y'_n,t'\right)=\\
\left(\left(W(t)\otimes_{1}W(t)\ldots\otimes_{(n-1)}W(t)
\right)\otimes\left(\bar W(t)\otimes_{1}\bar
W(t)\ldots\otimes_{(n-1)}\bar W(t) \right)\right)\cdot\xi\left(\vec
x_1,\ldots,\vec x_n,\vec y_1,\ldots,\vec y_n,t\right).
\end{multline}
Therefore, since we obtained before, that equation
\er{vhfffngghkjgghfjjghghghhjghjgghkghggkghghjghSYSPNjjjhhhjjkjkjklkjljljkkhhkjiijjijhbmbgghghdd1188}
is equivalent to the primed version of
\er{vhfffngghkjgghfjjghghghhjghjgghkghggkghghjghSYSPNjjjhhhjjaa} and
at the same time equation
\er{vhfffngghkjgghfjjghghghhjghjgghkghggkghghjghSYSPNjjjhhhjjaa} is
invariant under the change of non-inertial cartesian coordinate
system, provided we have \er{vyfgjhgjhvhgghSYSPNjj11jhjggj88}, with
the help of
\er{vhfffngghkjgghfjjghghghhjghjgghkghggkghghjghSYSPNjjjhhhjjkjkjklkjljljkkhhkjiijjijhbmbgghgh11ggghhgf11gg1188},
as before, we deduce
\begin{equation}\label{vhfffngghkjgghfjjghghghhjghjgghkghggkghghjghSYSPNjjjhhhjjkjkjklkjljljkkhhkjiijjijhbmbgghgh1188}
i\hbar\frac{\partial}{\partial t}\left(f\circ \hat H_{\vec
u}\right)= \hat H_0\cdot\left(f\circ \hat H_{\vec
u}\right)-\left(f\circ \hat H_{\vec u}\right)\cdot\hat H_0
\end{equation}
in an arbitrary coordinate system. So we get
\er{vhfffngghkjgghfjjghghghhjghjgghkghggkghghjghSYSPNjjjhhhjjkjkjklkjljljkkhhkjiijjijhbmbgghghihihhccnnnmvcc}
or equivalently
\er{vhfffngghkjgghfjjghghghhjghjgghkghggkghghjghSYSPNjjjhhhjjccnnncc}.

Finally, assume that the first and the second particles have the
same mass $m_1=m_2$ and the same charge $\sigma_1=\sigma_2$ and
moreover, assume that we have
$$V\left(\vec x_1,\vec x_2,\vec x_3,\ldots,\vec
x_n,t\right)=V\left(\vec x_2,\vec x_1,\vec x_3,\ldots,\vec
x_n,t\right)\quad\quad\forall\,\vec x_1,\vec x_2,\vec
x_3,\ldots,\vec x_n\in\mathbb{R}^3,\quad\forall t.$$ In this case
clearly by
\er{vhfffngghkjgghfjjghghghhjghjgghkghgghjhggjjkgSYSPNjjjljjkljjkjkjj88}
we deduce the following relation:
\begin{multline}\label{fyfgjguyiuiHHgjgjgghbmgjgfhfhhHHljjjhjhjhyyuyuhhh88}
B_{1,2}\cdot\psi(\vec x_2,\vec x_1,\vec x_3,\ldots,\vec
x_n,t)=-\psi(\vec x_1,\vec x_2,\vec x_3,\ldots,\vec
x_n,t)\;\;\forall\,\vec x_1,\ldots,\vec
x_n\in\mathbb{R}^3\quad\text{and}\quad \phi=\hat H_{\vec u}\cdot\psi\\
\quad\text{implies}\quad B_{1,2}\cdot\phi(\vec x_2,\vec x_1,\vec
x_3,\ldots,\vec x_n,t)=-\phi(\vec x_1,\vec x_2,\vec x_3,\ldots,\vec
x_n,t)\;\;\forall\,\vec x_1,\ldots,\vec x_n\in\mathbb{R}^3,
\end{multline}
where, as before, by $B_{1,2}:\mathbb{C}^{4^n}\to \mathbb{C}^{4^n}$
we denote the linear operator (matrix) defined as in
\er{fyfgjguyiuiHHgjgjgghhHHkbhbbnew} by the following:
\begin{equation}\label{fyfgjguyiuiHHgjgjgghhHHhuhuhuhhkh88}
B_{1,2}\cdot\left(a_1\otimes a_2\otimes a_3\otimes\ldots\otimes
a_n\right)=\left(a_2\otimes a_1\otimes a_3\otimes\ldots\otimes
a_n\right)\quad\quad\forall\, a_1,\ldots a_n\in\mathbb{C}^{4}.
\end{equation}
Therefore, by
\er{vhfffngghkjgghfjjghghghhjghjgghkghggkghghjghSYSPNjjjhhhjjkjkjklkjljljkkhhkjiijjihhjjhjhjhjhkkhhknnhkhh1188}
and \er{fyfgjguyiuiHHgjgjgghbmgjgfhfhhHHljjjhjhjhyyuyuhhh88} we
obtain
\begin{multline}\label{fyfgjguyiuiHHgjgjgghbmgjgfhfhhHHljjjhjhjhyyuyuhhhljjjkjj88}
B_{1,2}\cdot\psi(\vec x_2,\vec x_1,\vec x_3,\ldots,\vec
x_n,t)=-\psi(\vec x_1,\vec x_2,\vec x_3,\ldots,\vec
x_n,t)\;\;\forall\,\vec x_1,\ldots,\vec
x_n\in\mathbb{R}^3\quad\text{and}\quad \phi=\left(f\circ\hat H_{\vec u}\right)\cdot\psi\\
\quad\text{implies}\quad B_{1,2}\cdot\phi(\vec x_2,\vec x_1,\vec
x_3,\ldots,\vec x_n,t)=-\phi(\vec x_1,\vec x_2,\vec x_3,\ldots,\vec
x_n,t)\;\;\forall\,\vec x_1,\ldots,\vec x_n\in\mathbb{R}^3,
\end{multline}
consistently with
\er{fyfgjguyiuiHHgjgjgghbmgjgfhfhhHHljjjhjhjhyyuyuaa}.
\end{proof}

\subsection{Spinless relativistic-like particles: the Klein--Gordon
equation}\label{hggyugyuy1xx} Consider the motion of a system of $n$
stable relativistic-like spinless quantum micro-particles with
inertial masses $m_1,\ldots,m_n$ and the charges
$\sigma_1,\ldots,\sigma_n$ in the outer gravitational and
electromagnetical field with characteristics $\vec v(\vec x,t)$,
$\vec A(\vec x,t)$ and $\Psi(\vec x,t)$ and additional conservative
field with potential $V(\vec y_1,\ldots,\vec y_n,t)$. The
Klein--Gordon equation for this system of particles is is the
following:
\begin{multline}\label{vhfffngghkjgghfjjghghghhjghjgghkghgghjhggjjkgfgdiyfgfjkjgjgSYSPNxx}
\frac{\hbar^2}{2c^2\left(\sum\limits_{k=1}^{n}m_k\right)}\frac{\partial}{\partial
t}\left(\frac{\partial\psi}{\partial t}+\sum\limits_{k=1}^{n}\vec
v(\vec x_k,t)\cdot\nabla_{\vec
x_k}\psi+\sum\limits_{k=1}^{n}\frac{i\sigma_k}{\hbar}\Psi_0(\vec
x_k,t)\psi\right)\\+\sum_{j=1}^{n}\frac{\hbar^2}{2c^2\left(\sum\limits_{k=1}^{n}m_k\right)}\Div_{\vec
x_j}\left\{\vec v(\vec x_j,t)\left(\frac{\partial\psi}{\partial
t}+\sum\limits_{k=1}^{n}\vec v(\vec x_k,t)\cdot\nabla_{\vec
x_k}\psi+\sum\limits_{k=1}^{n}\frac{i\sigma_k}{\hbar}\Psi_0(\vec
x_k,t)\psi\right)\right\}
\\
+\frac{\hbar^2}{2c^2\left(\sum\limits_{k=1}^{n}m_k\right)}\left(\sum\limits_{k=1}^{n}\frac{i\sigma_k}{\hbar}\Psi_0(\vec
x_k,t)\right)\left(\frac{\partial\psi}{\partial
t}+\sum\limits_{k=1}^{n}\vec v(\vec x_k,t)\cdot\nabla_{\vec
x_k}\psi+\sum\limits_{k=1}^{n}\frac{i\sigma_k}{\hbar}\Psi_0(\vec
x_k,t)\psi\right)
\\
+\frac{c^2\left(\sum\limits_{k=1}^{n}m_k\right)}{2}\psi-\sum_{j=1}^{n}\frac{\hbar^2}{2m_j}\Delta_{\vec
x_j}\psi-V\left(\vec x_1,\ldots,\vec
x_n,t\right)\psi\\+\sum_{j=1}^{n}\frac{
i\hbar\sigma_j}{2m_jc}div_{\vec x_j}\left\{\psi\vec A(\vec
x_j,t)\right\}+\sum_{j=1}^{n}\frac{ i\hbar\sigma_j}{2m_jc}\vec
A(\vec x_j,t)\cdot\nabla_{\vec x_j}\psi+\sum_{j=1}^{n}\left(
\frac{\sigma^2_j}{2m_jc^2}\left|\vec A(\vec
x_j,t)\right|^2\right)\psi=0,
\end{multline}
where, as before,
\begin{equation}\label{noninchgravortbstrghgggSYSPNxxjjkkjxx}
\Psi_0(\vec x,t):=\Psi(\vec x,t)-\frac{1}{c}\vec v(\vec
x,t)\cdot\vec A(\vec x,t)
\end{equation}
is the proper electrical potential and $\psi=\psi\left(\vec
x_1,\ldots,\vec x_n,t\right)\in\mathbb{C}$ is the scalar wave
function. Next consider a change of some non-inertial cartesian
coordinate system $(*)$ to another cartesian coordinate system
$(**)$ of the form:
\begin{equation}\label{noninchgravortbstrghgggSYSPNxx}
\begin{cases}
\vec x'=A(t)\cdot\vec x+\vec z(t),\\
t'=t,
\end{cases}
\end{equation}
where $A(t)\in SO(3)$ is a rotation.
Then, we deduce that  the  Klein--Gordon  equation of the form
\er{vhfffngghkjgghfjjghghghhjghjgghkghgghjhggjjkgfgdiyfgfjkjgjgSYSPNxx}
is invariant under the change of non-inertial cartesian coordinate
system, provided that under \er{noninchgravortbstrghgggSYSPNxx} we
have
\begin{equation}\label{vyfgjhgjhvhgghSYSPNxx}
\begin{cases}
\psi'=\psi
\\
V'=V
\\
\vec A'=A(t)\cdot\vec A
\\
\Psi'_0:=\Psi'-\frac{1}{c}\vec A'\cdot\vec v'=\Psi-\frac{1}{c}\vec
A\cdot\vec v:=\Psi_0.
\end{cases}
\end{equation}

Next defining
\begin{equation}\label{noninchgravortbstrghgggSYSPNxxjjjkjkxx}
\psi_1\left(\vec x_1,\ldots,\vec
x_n,t\right):=e^{\frac{ic^2\left(\sum\limits_{k=1}^{n}m_k\right)t}{\hbar}}\psi\left(\vec
x_1,\ldots,\vec x_n,t\right)
\end{equation}
we rewrite
\er{vhfffngghkjgghfjjghghghhjghjgghkghgghjhggjjkgfgdiyfgfjkjgjgSYSPNxx}
as:
\begin{multline}\label{vhfffngghkjgghfjjghghghhjghjgghkghgghjhggjjkgfgdiyfgfjkjgjgSYSPNxxjkhhkxx}
\frac{c^2\left(\sum\limits_{k=1}^{n}m_k\right)}{2}\psi_1+\frac{\hbar^2}{2c^2\left(\sum\limits_{k=1}^{n}m_k\right)}e^{\frac{ic^2\left(\sum\limits_{k=1}^{n}m_k\right)t}{\hbar}}\frac{\partial}{\partial
t}\left(\left(\frac{\partial\psi_1}{\partial
t}+\sum\limits_{k=1}^{n}\vec v(\vec x_k,t)\cdot\nabla_{\vec
x_k}\psi_1\right)e^{-\frac{ic^2\left(\sum\limits_{k=1}^{n}m_k\right)t}{\hbar}}\right)
\\+
\frac{\hbar^2}{2c^2\left(\sum\limits_{k=1}^{n}m_k\right)}e^{\frac{ic^2\left(\sum\limits_{k=1}^{n}m_k\right)t}{\hbar}}\frac{\partial}{\partial
t}\left(\left(\sum\limits_{k=1}^{n}\frac{i\sigma_k}{\hbar}\Psi_0(\vec
x_k,t)\psi_1-\frac{ic^2\left(\sum\limits_{k=1}^{n}m_k\right)}{\hbar}\psi_1\right)e^{-\frac{ic^2\left(\sum\limits_{k=1}^{n}m_k\right)t}{\hbar}}\right)\\+\sum_{j=1}^{n}\frac{\hbar^2}{2c^2\left(\sum\limits_{k=1}^{n}m_k\right)}\Div_{\vec
x_j}\left\{\vec v(\vec x_j,t)\left(\frac{\partial\psi_1}{\partial t}
+\sum\limits_{k=1}^{n}\vec v(\vec x_k,t)\cdot\nabla_{\vec
x_k}\psi_1+\sum\limits_{k=1}^{n}\frac{i\sigma_k}{\hbar}\Psi_0(\vec
x_k,t)\psi_1\right)\right\}
\\
+\frac{\hbar^2}{2c^2\left(\sum\limits_{k=1}^{n}m_k\right)}\left(\sum\limits_{k=1}^{n}\frac{i\sigma_k}{\hbar}\Psi_0(\vec
x_k,t)\right)\left(\frac{\partial\psi_1}{\partial
t}+\sum\limits_{k=1}^{n}\vec v(\vec x_k,t)\cdot\nabla_{\vec
x_k}\psi_1+\sum\limits_{k=1}^{n}\frac{i\sigma_k}{\hbar}\Psi_0(\vec
x_k,t)\psi_1\right)
\\
-\sum_{j=1}^{n}\frac{\hbar^2}{2c^2\left(\sum\limits_{k=1}^{n}m_k\right)}\Div_{\vec
x_j}\left\{\vec v(\vec
x_j,t)\left(\frac{ic^2\left(\sum\limits_{k=1}^{n}m_k\right)}{\hbar}\psi_1\right)\right\}\\-\frac{\hbar^2}{2c^2\left(\sum\limits_{k=1}^{n}m_k\right)}\left(\sum\limits_{k=1}^{n}\frac{i\sigma_k}{\hbar}\Psi_0(\vec
x_k,t)\right)\left(\frac{ic^2\left(\sum\limits_{k=1}^{n}m_k\right)}{\hbar}\psi_1\right)
-\sum_{j=1}^{n}\frac{\hbar^2}{2m_j}\Delta_{\vec
x_j}\psi_1-V\left(\vec x_1,\ldots,\vec
x_n,t\right)\psi_1\\+\sum_{j=1}^{n}\frac{
i\hbar\sigma_j}{2m_jc}div_{\vec x_j}\left\{\psi_1\vec A(\vec
x_j,t)\right\}+\sum_{j=1}^{n}\frac{ i\hbar\sigma_j}{2m_jc}\vec
A(\vec x_j,t)\cdot\nabla_{\vec x_j}\psi_1+\sum_{j=1}^{n}\left(
\frac{\sigma^2_j}{2m_jc^2}\left|\vec A(\vec
x_j,t)\right|^2\right)\psi_1=0,
\end{multline}
that we further rewrite as:
\begin{multline}\label{vhfffngghkjgghfjjghghghhjghjgghkghgghjhggjjkgfgdiyfgfjkjgjgSYSPNxxjkhhkxxhjjhjhxx}
\frac{\hbar^2}{2c^2\left(\sum\limits_{k=1}^{n}m_k\right)}\frac{\partial}{\partial
t}\left(\frac{\partial\psi_1}{\partial t}+\sum\limits_{k=1}^{n}\vec
v(\vec x_k,t)\cdot\nabla_{\vec
x_k}\psi_1+\sum\limits_{k=1}^{n}\frac{i\sigma_k}{\hbar}\Psi_0(\vec
x_k,t)\psi_1\right)\\+\sum_{j=1}^{n}\frac{\hbar^2}{2c^2\left(\sum\limits_{k=1}^{n}m_k\right)}\Div_{\vec
x_j}\left\{\vec v(\vec x_j,t)\left(\frac{\partial\psi_1}{\partial
t}+\sum\limits_{k=1}^{n}\vec v(\vec x_k,t)\cdot\nabla_{\vec
x_k}\psi_1+\sum\limits_{k=1}^{n}\frac{i\sigma_k}{\hbar}\Psi_0(\vec
x_k,t)\psi_1\right)\right\}
\\
+\frac{\hbar^2}{2c^2\left(\sum\limits_{k=1}^{n}m_k\right)}\left(\sum\limits_{k=1}^{n}\frac{i\sigma_k}{\hbar}\Psi_0(\vec
x_k,t)\right)\left(\frac{\partial\psi_1}{\partial
t}+\sum\limits_{k=1}^{n}\vec v(\vec x_k,t)\cdot\nabla_{\vec
x_k}\psi_1+\sum\limits_{k=1}^{n}\frac{i\sigma_k}{\hbar}\Psi_0(\vec
x_k,t)\psi_1\right)
\\
-i\hbar\left(\frac{\partial\psi_1}{\partial
t}+\sum\limits_{k=1}^{n}\frac{1}{2}\vec v(\vec
x_k,t)\cdot\nabla_{\vec
x_k}\psi_1+\sum\limits_{k=1}^{n}\frac{1}{2}\Div_{\vec
x_k}\left\{\psi_1\vec v(\vec
x_k,t)\right\}+\sum\limits_{k=1}^{n}\frac{i\sigma_k}{\hbar}\Psi_0(\vec
x_k,t)\psi_1\right)\\-\sum_{j=1}^{n}\frac{\hbar^2}{2m_j}\Delta_{\vec
x_j}\psi_1-V\left(\vec x_1,\ldots,\vec
x_n,t\right)\psi_1\\+\sum_{j=1}^{n}\frac{
i\hbar\sigma_j}{2m_jc}div_{\vec x_j}\left\{\psi_1\vec A(\vec
x_j,t)\right\}+\sum_{j=1}^{n}\frac{ i\hbar\sigma_j}{2m_jc}\vec
A(\vec x_j,t)\cdot\nabla_{\vec x_j}\psi_1+\sum_{j=1}^{n}\left(
\frac{\sigma^2_j}{2m_jc^2}\left|\vec A(\vec
x_j,t)\right|^2\right)\psi_1=0.
\end{multline}
Thus by
\er{vhfffngghkjgghfjjghghghhjghjgghkghgghjhggjjkgfgdiyfgfjkjgjgSYSPNxxjkhhkxxhjjhjhxx}
in the non-relativistic limit we obtain:
\begin{multline}\label{vhfffngghkjgghfjjghghghhjghjgghkghgghjhggjjkgfgdiyfgfjkjgjgSYSPNjjj}
i\hbar\left(\frac{\partial\psi_1}{\partial t}+\sum_{j=1}^{n}\vec
v(\vec x_j,t)\cdot\nabla_{\vec
x_j}\psi_1\right)+\sum_{j=1}^{n}\frac{i\hbar}{2}\left(div_{\vec
x_j}\vec v(\vec x_j,t)\right)\psi_1\approx\\
-\sum_{j=1}^{n}\frac{\hbar^2}{2m_j}\Delta_{\vec
x_j}\psi_1-V\left(\vec x_1,\ldots,\vec
x_n,t\right)\psi_1+\sum_{j=1}^{n}\frac{
i\hbar\sigma_j}{2m_jc}div_{\vec x_j}\left\{\psi_1\vec A(\vec
x_j,t)\right\}+\sum_{j=1}^{n}\frac{ i\hbar\sigma_j}{2m_jc}\vec
A(\vec x_j,t)\cdot\nabla_{\vec
x_j}\psi_1\\+\sum_{j=1}^{n}\left(\sigma_j\Psi(\vec
x_j,t)-\frac{\sigma_j}{c}\vec A(\vec x_j,t)\cdot\vec v(\vec
x_j,t)+\frac{\sigma^2_j}{2m_jc^2}\left|\vec A(\vec
x_j,t)\right|^2\right)\psi_1,
\end{multline}
that coincides with the Schr\"{o}dinger equation of the form
\er{vhfffngghkjgghfjjghghghhjghjgghkghgghjhggjjkgfgdiyfgfjkjgjgSYSPNHH}.

Next, again consider the motion and interaction of system of $n$
stable quantum micro-particles having inertial masses $m_1,\ldots,
m_n$ and the charges $\sigma_1,\ldots,\sigma_n$ with the known
gravitational and electromagnetical field with potentials $\vec
v(\vec x,t)$, $\vec A(\vec x,t)$ and $\Psi(\vec x,t)$ and additional
conservative field with potential $V(\vec y_1,\ldots,\vec y_n,t)$.
Then consider a Lagrangian density $L$ defined by
\begin{multline}\label{vhfffngghkjgghDDmmkkksmfggffgZZxx}
L_2\left(\psi,\vec A,\Psi,\vec v,\vec x_1,\ldots,\vec
x_n,t\right):=-\frac{c^2\left(\sum\limits_{k=1}^{n}m_k\right)}{2}\psi\cdot\bar\psi-\sum\limits_{k=1}^{n}\frac{\hbar^2}{2m_k}\nabla_{\vec
x_k}\psi\cdot\nabla_{\vec
x_k}\bar\psi\\
-\sum\limits_{k=1}^{n}\frac{\hbar\sigma_k
i}{2m_kc}\left(\nabla_{\vec
x_k}\psi\cdot\bar\psi-\psi\cdot\nabla_{\vec
x_k}\bar\psi\right)\cdot\vec A(\vec
x_k,t)-\sum\limits_{k=1}^{n}\frac{\sigma^2_k}{2m_kc^2}\left|\vec
A(\vec x_k,t)\right|^2\psi\cdot\bar\psi+V\left(\vec x_1,\ldots,\vec
x_n,t\right)\psi\cdot\bar\psi
\\
+\frac{\hbar^2}{2c^2\left(\sum\limits_{k=1}^{n}m_k\right)}\left(\frac{\partial\psi}{\partial
t}+\sum\limits_{k=1}^{n}\vec v(\vec x_k,t)\cdot\nabla_{\vec
x_k}\psi+\sum\limits_{k=1}^{n}\frac{i\sigma_k}{\hbar}\Psi_0(\vec
x_k,t)\psi\right)\times\\
\times\left(\frac{\partial\bar\psi}{\partial
t}+\sum\limits_{k=1}^{n}\vec v(\vec x_k,t)\cdot\nabla_{\vec
x_k}\bar\psi-\sum\limits_{k=1}^{n}\frac{i\sigma_k}{\hbar}\Psi_0(\vec
x_k,t)\bar\psi\right),
\end{multline}
where $\psi(\vec x_1,\ldots,\vec x_n,t)\in\mathbb{C}$ is a wave
function of the system. Then, as before, it can be proven that $L$
is invariant under the change of inertial or non-inertial cartesian
coordinate systems of the form
\begin{equation*}
\begin{cases}
t'=t\\
\vec x'_k=A(t)\cdot\vec x_k+\vec z(t)\quad\quad\forall
k=1,\ldots,n,\end{cases}
\end{equation*}
provided that $\psi'=\psi$. We investigate stationary points of the
functional
\begin{equation}\label{btfffygtgyggyDDmmkkksmZZxx}
J=\int_0^T\int_{\left(\mathbb{R}^3\right)^n}L\left(\psi,\vec
A,\Psi,\vec v,\vec x_1,\ldots,\vec x_n,t\right)d\vec x_1\ldots,d\vec
x_n dt.
\end{equation}
Then,
\begin{multline}\label{vhfffngghkjgghDDmmkkkghgsmZZxx}
0=\frac{\delta L_2}{\delta (\bar\psi)}=
-\frac{\hbar^2}{2c^2\left(\sum\limits_{k=1}^{n}m_k\right)}\frac{\partial}{\partial
t}\left(\frac{\partial\psi}{\partial t}+\sum\limits_{k=1}^{n}\vec
v(\vec x_k,t)\cdot\nabla_{\vec
x_k}\psi+\sum\limits_{k=1}^{n}\frac{i\sigma_k}{\hbar}\Psi_0(\vec
x_k,t)\psi\right)\\-\sum_{j=1}^{n}\frac{\hbar^2}{2c^2\left(\sum\limits_{k=1}^{n}m_k\right)}\Div_{\vec
x_j}\left\{\vec v(\vec x_j,t)\left(\frac{\partial\psi}{\partial
t}+\sum\limits_{k=1}^{n}\vec v(\vec x_k,t)\cdot\nabla_{\vec
x_k}\psi+\sum\limits_{k=1}^{n}\frac{i\sigma_k}{\hbar}\Psi_0(\vec
x_k,t)\psi\right)\right\}
\\
-\frac{\hbar^2}{2c^2\left(\sum\limits_{k=1}^{n}m_k\right)}\left(\sum\limits_{k=1}^{n}\frac{i\sigma_k}{\hbar}\Psi_0(\vec
x_k,t)\right)\left(\frac{\partial\psi}{\partial
t}+\sum\limits_{k=1}^{n}\vec v(\vec x_k,t)\cdot\nabla_{\vec
x_k}\psi+\sum\limits_{k=1}^{n}\frac{i\sigma_k}{\hbar}\Psi_0(\vec
x_k,t)\psi\right)
\\
-\frac{c^2\left(\sum\limits_{k=1}^{n}m_k\right)}{2}\psi+\sum\limits_{k=1}^{n}\frac{\hbar^2}{2m_k}\Delta_{\vec
x_k}\psi -\sum\limits_{k=1}^{n}\frac{\hbar\sigma_k
i}{2m_kc}\left(\vec A(\vec x_k,t)\cdot\nabla_{\vec
x_k}\psi+div_{\vec x_k}\left\{\psi\vec A(\vec
x_k,t)\right\}\right)\\
-\sum\limits_{k=1}^{n}\frac{\sigma^2_k}{2m_kc^2}\left|\vec A(\vec
x_k,t)\right|^2\psi
+V\left(\vec x_1,\ldots,\vec x_n,t\right)\psi,
\end{multline}
and
\begin{multline}\label{vhfffngghkjgghDDmmkkkghguhyuysmZZxx}
0=\frac{\delta L_2}{\delta (\psi)}=
-\frac{\hbar^2}{2c^2\left(\sum\limits_{k=1}^{n}m_k\right)}\frac{\partial}{\partial
t}\left(\frac{\partial\bar\psi}{\partial
t}+\sum\limits_{k=1}^{n}\vec v(\vec x_k,t)\cdot\nabla_{\vec
x_k}\bar\psi+\sum\limits_{k=1}^{n}\frac{\bar
{(i)}\sigma_k}{\hbar}\Psi_0(\vec
x_k,t)\bar\psi\right)\\-\sum_{j=1}^{n}\frac{\hbar^2}{2c^2\left(\sum\limits_{k=1}^{n}m_k\right)}\Div_{\vec
x_j}\left\{\vec v(\vec x_j,t)\left(\frac{\partial\bar\psi}{\partial
t}+\sum\limits_{k=1}^{n}\vec v(\vec x_k,t)\cdot\nabla_{\vec
x_k}\bar\psi+\sum\limits_{k=1}^{n}\frac{\bar
{(i)}\sigma_k}{\hbar}\Psi_0(\vec x_k,t)\bar\psi\right)\right\}
\\
-\frac{\hbar^2}{2c^2\left(\sum\limits_{k=1}^{n}m_k\right)}\left(\sum\limits_{k=1}^{n}\frac{\bar
{(i)}\sigma_k}{\hbar}\Psi_0(\vec
x_k,t)\right)\left(\frac{\partial\bar\psi}{\partial
t}+\sum\limits_{k=1}^{n}\vec v(\vec x_k,t)\cdot\nabla_{\vec
x_k}\bar\psi+\sum\limits_{k=1}^{n}\frac{\bar
{(i)}\sigma_k}{\hbar}\Psi_0(\vec x_k,t)\bar\psi\right)\\
-\frac{c^2\left(\sum\limits_{k=1}^{n}m_k\right)}{2}\bar\psi+\sum\limits_{k=1}^{n}\frac{\hbar^2}{2m_k}\Delta_{\vec
x_k}\bar\psi -\sum\limits_{k=1}^{n}\frac{\hbar\sigma_k \bar
{(i)}}{2m_kc}\left(\vec A(\vec x_k,t)\cdot\nabla_{\vec
x_k}\bar\psi+div_{\vec x_k}\left\{\bar\psi\vec A(\vec
x_k,t)\right\}\right)\\-\sum\limits_{k=1}^{n}\frac{\sigma^2_k}{2m_kc^2}\left|\vec
A(\vec x_k,t)\right|^2\bar\psi
+V\left(\vec x_1,\ldots,\vec x_n,t\right)\bar\psi,
\end{multline}
where the last equality is just the complex conjugate of
\er{vhfffngghkjgghDDmmkkkghgsmZZxx}. So we get that the
Euler-Lagrange equation for \er{btfffygtgyggyDDmmkkksmZZxx}
coincides with the Klein--Gordon  equation of the form
\er{vhfffngghkjgghfjjghghghhjghjgghkghgghjhggjjkgfgdiyfgfjkjgjgSYSPNxx}.

Finally, assume that the first and the second particles have the
same mass $m_1=m_2$ and the same charge $\sigma_1=\sigma_2$ and
moreover, assume that we have
$$V\left(\vec x_1,\vec x_2,\vec x_3,\ldots,\vec
x_n,t\right)=V\left(\vec x_2,\vec x_1,\vec x_3,\ldots,\vec
x_n,t\right)\quad\quad\forall\,\vec x_1,\vec x_2,\vec
x_3,\ldots,\vec x_n\in\mathbb{R}^3,\quad\forall t.$$ In this case it
can be easily deduced that if  $\psi\left(\vec x_1,\vec x_2,\vec
x_3,\ldots,\vec x_n,t\right)$ is a solution of
\er{vhfffngghkjgghfjjghghghhjghjgghkghgghjhggjjkgfgdiyfgfjkjgjgSYSPNxx},
then $\psi(\vec x_2,\vec x_1,\vec x_3,\ldots,\vec x_n,t)$ as also a
solutions of
\er{vhfffngghkjgghfjjghghghhjghjgghkghgghjhggjjkgfgdiyfgfjkjgjgSYSPNxx},
Therefore, if $\psi\left(\vec x_1,\vec x_2,\vec x_3,\ldots,\vec
x_n,t\right)$ is a solution of
\er{vhfffngghkjgghfjjghghghhjghjgghkghgghjhggjjkgfgdiyfgfjkjgjgSYSPNxx},
then for every $t\geq 0$ we will have
\begin{multline}\label{fyfgjguyiuiHHgjgjgghbmgjgfhfhhHHrr22ggkk}
\psi(\vec x_2,\vec x_1,\vec x_3,\ldots,\vec x_n,t)=\psi\left(\vec
x_1,\vec x_2,\vec x_3,\ldots,\vec x_n,t\right)\quad\text{and}\quad\\
\frac{\partial\psi}{\partial t}(\vec x_2,\vec x_1,\vec
x_3,\ldots,\vec x_n,t)=\frac{\partial\psi}{\partial t}\left(\vec
x_1,\vec x_2,\vec x_3,\ldots,\vec x_n,t\right)\quad\forall\,\vec
x_1,\vec x_2,\vec x_3,\ldots,\vec x_n\in\mathbb{R}^3,
\end{multline}
provided that \er{fyfgjguyiuiHHgjgjgghbmgjgfhfhhHHrr22ggkk} holds
for the initial instant of time $t=0$. So we have a consistency with
the principles of identity for two or more identical bosons, in the
case of spinless particles.

\section{Thermodynamics of a moving continuum medium}\label{gjgggggggggggggggggggggggghkgu}
Again, consistently with
\er{noninchgravortbstrjgghguittu2gjgghhjhghjhjgghgghghghtytythvfghfgghjgg},
consider in some cartesian coordinate system $(*)$ the second Law of
Newton for the moving continuum medium with the inertial mass
density $\mu$, the field of average (macroscopic) velocities $\vec
u$, the charge density $\rho$ and the electric current density $\vec
j$:
\begin{multline}\label{MaxVacFull1ninshtrgravortghhghgjkgghklhjgkghghjjkjhjkkggjkhjkhjjhhfhjhkjkhbbgjhzzzyyykkknnnNWBWHWPPNjj}
\mu\left(\frac{\partial\vec u}{\partial t}+d_{\vec x}\vec u\cdot
\vec u\right)=\frac{\partial}{\partial t}(\mu\vec u)+\Div_{\vec
x}\left\{\mu\vec u\otimes\vec u\right\}=\\-\mu\vec u\times
curl_{\vec x}\vec v+\mu\left(\partial_{t}\vec
v+\nabla_{\vec x}\frac{1}{2}|\vec v|^2\right)+\rho\vec E+\frac{1}{c}\vec j\times\vec B+\Div_{\vec x}\mathcal{T}=\\
=-\mu(\vec u-\vec v)\times curl_{\vec x}\vec v
+\mu\left(\partial_t\vec v+ d_{\vec x}\vec v\cdot\vec
v\right)+\rho\vec E+\frac{1}{c}\vec j\times\vec B+\Div_{\vec
x}\mathcal{T}.
\end{multline}
Here $\rho\vec E+\frac{1}{c}\vec j\times\vec B$ is the volume
density of the Lorentz force where $\vec E$ and $\vec B$ are outer
electric and magnetic fields, assumed to be changing smoothly and
almost constant in the microscopic level, $\vec v$ is a vectorial
gravitational potential also assumed to be changing smoothly and
almost constant in the microscopic level, and $\mathcal{T}\in
\R^{3\times 3}$ is the symmetric Cauchy stress tensor of the
continuum medium. Moreover, the mass density $\mu$, clearly
satisfies the continuum equation:
\begin{equation}\label{MaxVacFull1ninshtrgravortghhghgjkgghklhjgkghghjjkjhjkkggjkhjkhjjhhfhjhkhjjkhjgzzzyyykkknnnNWBWHWPPNjj}
\frac{\partial\mu}{\partial t}+div_{\vec x}\left(\mu\vec u\right)=0.
\end{equation}
In particular, multiplying
\er{MaxVacFull1ninshtrgravortghhghgjkgghklhjgkghghjjkjhjkkggjkhjkhjjhhfhjhkjkhbbgjhzzzyyykkknnnNWBWHWPPNjj}
by $\vec u$ and using
\er{MaxVacFull1ninshtrgravortghhghgjkgghklhjgkghghjjkjhjkkggjkhjkhjjhhfhjhkhjjkhjgzzzyyykkknnnNWBWHWPPNjj}
we deduce the equality of the balance of the kinetic energy:
\begin{multline}\label{MaxVacFull1ninshtrgravortghhghgjkgghklhjgkghghjjkjhjkkggjkhjkhjjhhfhjhkjkhbbgjhzzzyyykkknnnNWBWHWPPNjjnkjhbbj}
\frac{\partial}{\partial t}\left(\frac{\mu}{2}\left|\vec
u\right|^2\right)+\Div_{\vec x}\left\{\left(\frac{\mu}{2}\left|\vec
u\right|^2\right)\vec u\right\}=\mu\left(\frac{\partial}{\partial
t}\left(\frac{1}{2}\left|\vec u\right|^2\right)+\vec
u\cdot\nabla_{\vec x}\left(\frac{1}{2}\left|\vec
u\right|^2\right)\right)=\\ \mu\left(\partial_{t}\vec v+\nabla_{\vec
x}\frac{1}{2}|\vec v|^2\right)\cdot\vec u+\left(\rho\vec
E+\frac{1}{c}\vec j\times\vec B\right)\cdot\vec u+\left(\Div_{\vec
x}\mathcal{T}\right)\cdot\vec u =\\ \mu\left(\partial_{t}\vec
v+\nabla_{\vec x}\frac{1}{2}|\vec v|^2\right)\cdot\vec u+\rho\vec
E\cdot\vec u-\frac{1}{c}\left(\vec u\times\vec B\right)\cdot\vec
j+\left(\Div_{\vec x}\mathcal{T}\right)\cdot\vec u.
\end{multline}

Next, it is well known (see \cite{SM}), that the First Law of
Thermodynamics of this moving medium has the following form:
\begin{multline}\label{MaxVacFull1ninshtrgravortghhghgjkgghklhjgkghghjjkjhjkkggjkhjkhjjhhfhjhkjkhbbgjhzzzyyykkknnnNWBWHWPPNjjhjhjhjh}
\frac{\partial E}{\partial t}+\Div_{\vec x}\left\{E\vec
u\right\}=\mu\left(\frac{\partial}{\partial
t}\left(\frac{E}{\mu}\right)+\vec u\cdot\nabla_{\vec
x}\left(\frac{E}{\mu}\right)\right)\\ =\frac{1}{2}\left(d_{\vec
x}\vec u+\left\{d_{\vec x}\vec
u\right\}^T\right):\mathcal{T}-\Div_{\vec x}\vec q+\left(\vec
j-\rho\vec u\right)\cdot\left(\vec E+\frac{1}{c}\vec u\times\vec
B\right).
\end{multline}
Here $E$ is the volume density of the internal energy (energy per
unit volume) and consistently $\frac{E}{\mu}$ is the internal energy
per unit mass, $\vec q$ is the heat flux and
$$\left(\vec j-\rho\vec u\right)\cdot\left(\vec E+\frac{1}{c}\vec
u\times\vec B\right)$$ is the Joules heat term. In particular,
adding
\er{MaxVacFull1ninshtrgravortghhghgjkgghklhjgkghghjjkjhjkkggjkhjkhjjhhfhjhkjkhbbgjhzzzyyykkknnnNWBWHWPPNjjhjhjhjh}
with
\er{MaxVacFull1ninshtrgravortghhghgjkgghklhjgkghghjjkjhjkkggjkhjkhjjhhfhjhkjkhbbgjhzzzyyykkknnnNWBWHWPPNjjnkjhbbj}
and using the symmetry of $\mathcal{T}$, we deduce the following
equality of the balance of the energy:
\begin{multline}\label{MaxVacFull1ninshtrgravortghhghgjkgghklhjgkghghjjkjhjkkggjkhjkhjjhhfhjhkjkhbbgjhzzzyyykkknnnNWBWHWPPNjjnkjhbbjhjjhjh}
\frac{\partial}{\partial t}\left(\frac{\mu}{2}\left|\vec
u\right|^2+E\right)+\Div_{\vec
x}\left\{\left(\frac{\mu}{2}\left|\vec u\right|^2+E\right)\vec
u\right\}=\mu\left(\frac{\partial}{\partial
t}\left(\frac{1}{2}\left|\vec u\right|^2+\frac{E}{\mu}\right)+\vec
u\cdot\nabla_{\vec x}\left(\frac{1}{2}\left|\vec
u\right|^2+\frac{E}{\mu}\right)\right)\\= \Div_{\vec
x}\left(\mathcal{T}\cdot\vec u\right)-\Div_{\vec x}\vec q+
\mu\left(\partial_{t}\vec v+\nabla_{\vec x}\frac{1}{2}|\vec
v|^2\right)\cdot\vec u+\vec E\cdot\vec j.
\end{multline}

Next the Second Law of Thermodynamics states that
\begin{equation}\label{MaxVacFull1ninshtrgravortghhghgjkgghklhjgkghghjjkjhjkkggjkhjkhjjhhfhjhkjkhbbgjhzzzyyykkknnnNWBWHWPPNjjhjhjhjhkpi}
T\left(\frac{\partial \mathcal{S}}{\partial t}+\Div_{\vec
x}\left\{\mathcal{S}\vec
u\right\}\right)=T\mu\left(\frac{\partial}{\partial
t}\left(\frac{\mathcal{S}}{\mu}\right)+\vec u\cdot\nabla_{\vec
x}\left(\frac{\mathcal{S}}{\mu}\right)\right)\geq -\Div_{\vec x}\vec
q.
\end{equation}
Here $T:=T(\vec x,t)$ is the Kelvin temperature field and
$\mathcal{S}$ is the volume density of the entropy (entropy per unit
volume) and consistently $\frac{\mathcal{S}}{\mu}$ is the entropy
per unit mass. Moreover, we have the equality in
\er{MaxVacFull1ninshtrgravortghhghgjkgghklhjgkghghjjkjhjkkggjkhjkhjjhhfhjhkjkhbbgjhzzzyyykkknnnNWBWHWPPNjjhjhjhjhkpi}
in the case of reversible or quasi-reversible process. In the latter
case we rewrite the First Law
\er{MaxVacFull1ninshtrgravortghhghgjkgghklhjgkghghjjkjhjkkggjkhjkhjjhhfhjhkjkhbbgjhzzzyyykkknnnNWBWHWPPNjjhjhjhjh}
and the Second Law
\er{MaxVacFull1ninshtrgravortghhghgjkgghklhjgkghghjjkjhjkkggjkhjkhjjhhfhjhkjkhbbgjhzzzyyykkknnnNWBWHWPPNjjhjhjhjhkpi}
together as:
\begin{multline}\label{MaxVacFull1ninshtrgravortghhghgjkgghklhjgkghghjjkjhjkkggjkhjkhjjhhfhjhkjkhbbgjhzzzyyykkknnnNWBWHWPPNjjhjhjhjhkloioj}
\frac{\partial E}{\partial t}+\Div_{\vec x}\left\{E\vec u\right\}
=T\left(\frac{\partial \mathcal{S}}{\partial t}+\Div_{\vec
x}\left\{\mathcal{S}\vec u\right\}\right)+\frac{1}{2}\left(d_{\vec
x}\vec u+\left\{d_{\vec x}\vec
u\right\}^T\right):\mathcal{T}+\left(\vec j-\rho\vec
u\right)\cdot\left(\vec E+\frac{1}{c}\vec u\times\vec B\right)\\=
T\mu\left(\frac{\partial}{\partial
t}\left(\frac{\mathcal{S}}{\mu}\right)+\vec u\cdot\nabla_{\vec
x}\left(\frac{\mathcal{S}}{\mu}\right)\right)+
\frac{1}{2}\left(d_{\vec x}\vec u+\left\{d_{\vec x}\vec
u\right\}^T\right):\mathcal{T}+\left(\vec j-\rho\vec
u\right)\cdot\left(\vec E+\frac{1}{c}\vec u\times\vec B\right)
\\=\mu\left(\frac{\partial}{\partial
t}\left(\frac{E}{\mu}\right)+\vec u\cdot\nabla_{\vec
x}\left(\frac{E}{\mu}\right)\right) .
\end{multline}
In particular, if the stress tensor have the following particular
form
\begin{equation}\label{MaxVacFull1ninshtrgravortghhghgjkgghklhjgkghghjjkjhjkkggjkhjkhjjhhfhjhkhjjkhjgzzzyyykkknnnNWBWHWPPNjjhhjhj}
\mathcal{T}=-p\,I,
\end{equation}
where $p$ is the scalar pressure and $I:=Id\in\R^{3\times 3}$ is the
identity matrix, then since by
\er{MaxVacFull1ninshtrgravortghhghgjkgghklhjgkghghjjkjhjkkggjkhjkhjjhhfhjhkhjjkhjgzzzyyykkknnnNWBWHWPPNjjhhjhj}
and
\er{MaxVacFull1ninshtrgravortghhghgjkgghklhjgkghghjjkjhjkkggjkhjkhjjhhfhjhkhjjkhjgzzzyyykkknnnNWBWHWPPNjj}
we have
\begin{multline}\label{MaxVacFull1ninshtrgravortghhghgjkgghklhjgkghghjjkjhjkkggjkhjkhjjhhfhjhkjkhbbgjhzzzyyykkknnnNWBWHWPPNjjhjhjhjhkloiojjojj}
\frac{1}{2}\left(d_{\vec x}\vec u+\left\{d_{\vec x}\vec
u\right\}^T\right):\mathcal{T}=-p\left(\Div_{\vec x}\vec
u\right)=-p\mu\left(-\frac{1}{\mu^2}\right)\left(\frac{\partial\mu}{\partial
t}+\vec u\cdot\nabla_{\vec x}\mu\right)
\\=-p\mu\left(\frac{\partial}{\partial
t}\left(\frac{1}{\mu}\right)+\vec u\cdot\nabla_{\vec
x}\left(\frac{1}{\mu}\right)\right),
\end{multline}
in the latter case we rewrite the First Law of Thermodynamics
\er{MaxVacFull1ninshtrgravortghhghgjkgghklhjgkghghjjkjhjkkggjkhjkhjjhhfhjhkjkhbbgjhzzzyyykkknnnNWBWHWPPNjjhjhjhjh}
as:
\begin{multline}\label{MaxVacFull1ninshtrgravortghhghgjkgghklhjgkghghjjkjhjkkggjkhjkhjjhhfhjhkjkhbbgjhzzzyyykkknnnNWBWHWPPNjjhjhjhjhlllkl;kl}
\frac{\partial}{\partial t}\left(\frac{E}{\mu}\right)+\vec
u\cdot\nabla_{\vec x}\left(\frac{E}{\mu}\right)
=-p\left(\frac{\partial}{\partial t}\left(\frac{1}{\mu}\right)+\vec
u\cdot\nabla_{\vec
x}\left(\frac{1}{\mu}\right)\right)-\frac{1}{\mu}\Div_{\vec x}\vec
q+\frac{1}{\mu}\left(\vec j-\rho\vec u\right)\cdot\left(\vec
E+\frac{1}{c}\vec u\times\vec B\right),
\end{multline}
and in the case of quasi-reversible process we rewrite
\er{MaxVacFull1ninshtrgravortghhghgjkgghklhjgkghghjjkjhjkkggjkhjkhjjhhfhjhkjkhbbgjhzzzyyykkknnnNWBWHWPPNjjhjhjhjhkloioj}
as:
\begin{multline}\label{MaxVacFull1ninshtrgravortghhghgjkgghklhjgkghghjjkjhjkkggjkhjkhjjhhfhjhkjkhbbgjhzzzyyykkknnnNWBWHWPPNjjhjhjhjhlllklkljkkhhkgk}
\frac{\partial}{\partial t}\left(\frac{E}{\mu}\right)+\vec
u\cdot\nabla_{\vec x}\left(\frac{E}{\mu}\right)
=-p\left(\frac{\partial}{\partial t}\left(\frac{1}{\mu}\right)+\vec
u\cdot\nabla_{\vec
x}\left(\frac{1}{\mu}\right)\right)+T\left(\frac{\partial}{\partial
t}\left(\frac{\mathcal{S}}{\mu}\right)+\vec u\cdot\nabla_{\vec
x}\left(\frac{\mathcal{S}}{\mu}\right)\right)\\+\frac{1}{\mu}\left(\vec
j-\rho\vec u\right)\cdot\left(\vec E+\frac{1}{c}\vec u\times\vec
B\right),
\end{multline}
where clearly $\frac{1}{\mu}$ is the volume per unit mass. Moreover,
the following inequality always holds:
\begin{equation}\label{MaxVacFull1ninshtrgravortghhghgjkgghklhjgkghghjjkjhjkkggjkhjkhjjhhfhjhkhjjkhjgzzzyyykkknnnNWBWHWPPNjjlkljjkkkjkjk}
\vec q\cdot\nabla_{\vec x}T\leq 0.
\end{equation}
In particular, by inserting
\er{MaxVacFull1ninshtrgravortghhghgjkgghklhjgkghghjjkjhjkkggjkhjkhjjhhfhjhkhjjkhjgzzzyyykkknnnNWBWHWPPNjjlkljjkkkjkjk}
into
\er{MaxVacFull1ninshtrgravortghhghgjkgghklhjgkghghjjkjhjkkggjkhjkhjjhhfhjhkjkhbbgjhzzzyyykkknnnNWBWHWPPNjjhjhjhjhkpi}
we deduce
\begin{multline}\label{MaxVacFull1ninshtrgravortghhghgjkgghklhjgkghghjjkjhjkkggjkhjkhjjhhfhjhkjkhbbgjhzzzyyykkknnnNWBWHWPPNjjhjhjhjhkpihjjgjggj}
\frac{\partial \mathcal{S}}{\partial t}+\Div_{\vec
x}\left\{\mathcal{S}\vec u\right\}=\mu\left(\frac{\partial}{\partial
t}\left(\frac{\mathcal{S}}{\mu}\right)+\vec u\cdot\nabla_{\vec
x}\left(\frac{\mathcal{S}}{\mu}\right)\right)\geq
-\frac{1}{T}\Div_{\vec x}\vec q\geq\\ -\frac{1}{T}\Div_{\vec x}\vec
q+\frac{1}{T^2}\nabla_{\vec x}T\cdot\vec q=-\Div_{\vec
x}\left(\frac{1}{T}\,\vec q\right),
\end{multline}
i.e.
\begin{equation}\label{MaxVacFull1ninshtrgravortghhghgjkgghklhjgkghghjjkjhjkkggjkhjkhjjhhfhjhkjkhbbgjhzzzyyykkknnnNWBWHWPPNjjhjhjhjhkpihjjgjggjkhhh}
\frac{\partial \mathcal{S}}{\partial t}+\Div_{\vec
x}\left\{\mathcal{S}\vec u\right\}=\mu\left(\frac{\partial}{\partial
t}\left(\frac{\mathcal{S}}{\mu}\right)+\vec u\cdot\nabla_{\vec
x}\left(\frac{\mathcal{S}}{\mu}\right)\right)\geq -\Div_{\vec
x}\left(\frac{1}{T}\,\vec q\right).
\end{equation}
Note that by
\er{MaxVacFull1ninshtrgravortghhghgjkgghklhjgkghghjjkjhjkkggjkhjkhjjhhfhjhkjkhbbgjhzzzyyykkknnnNWBWHWPPNjjhjhjhjhkpihjjgjggjkhhh},
the total entropy of an arbitrary moving \underline{thermally
isolated} system is a non-decreasing function of time. Indeed,
integrating
\er{MaxVacFull1ninshtrgravortghhghgjkgghklhjgkghghjjkjhjkkggjkhjkhjjhhfhjhkjkhbbgjhzzzyyykkknnnNWBWHWPPNjjhjhjhjhkpihjjgjggjkhhh}
on the region $\Omega(t)$, occupied by our system at the instant of
time $t$, and using the Divergence Theorem, gives:
\begin{multline}\label{MaxVacFull1ninshtrgravortghhghgjkgghklhjgkghghjjkjhjkkggjkhjkhjjhhfhjhkjkhbbgjhzzzyyykkknnnNWBWHWPPNjjhjhjhjhkpihjjgjggjkhhhjgjghghhjjh}
\frac{d}{dt}\left(\int_{\Omega(t)}\mathcal{S}(\vec x,t)\,d\vec
x\right)=\int_{\Omega(t)}\frac{\partial \mathcal{S}}{\partial
t}(\vec x,t)\,d\vec x+\int_{\partial\Omega(t)}\mathcal{S}(\vec
x,t)\,\vec u(\vec x,t)\cdot\vec n(\vec x,t)\,ds_{\vec x}
\\=\int_{\Omega(t)}\mu\left(\frac{\partial}{\partial
t}\left(\frac{\mathcal{S}}{\mu}\right)+\vec u\cdot\nabla_{\vec
x}\left(\frac{\mathcal{S}}{\mu}\right)\right)d\vec x
 \geq
-\int_{\partial\Omega(t)}\frac{1}{T(\vec x,t)}\,\vec q(\vec
x,t)\cdot\vec n(\vec x,t)\,ds_{\vec x},
\end{multline}
where $\vec n:=\vec n(\vec x,t)$ is an outward normal to the
boundary $\partial\Omega(t)$ of the region $\Omega(t)$. On the other
hand, as our system is thermally isolated, the heat flux $\vec q$ on
the boundary of $\Omega(t)$ vanishes. So, the right hand side of
inequality
\er{MaxVacFull1ninshtrgravortghhghgjkgghklhjgkghghjjkjhjkkggjkhjkhjjhhfhjhkjkhbbgjhzzzyyykkknnnNWBWHWPPNjjhjhjhjhkpihjjgjggjkhhhjgjghghhjjh}
is negligible and thus the total entropy of the system is a
non-decreasing function of time.

Finally, we remind the approximate Fourier's law (which is
consistent with
\er{MaxVacFull1ninshtrgravortghhghgjkgghklhjgkghghjjkjhjkkggjkhjkhjjhhfhjhkhjjkhjgzzzyyykkknnnNWBWHWPPNjjlkljjkkkjkjk}):
\begin{equation}\label{MaxVacFull1ninshtrgravortghhghgjkgghklhjgkghghjjkjhjkkggjkhjkhjjhhfhjhkhjjkhjgzzzyyykkknnnNWBWHWPPNjjlkljj}
\vec q=-\chi\nabla_{\vec x}T,
\end{equation}
where $\chi$ is some positive material coefficient (not necessary a
constant). The generalization of
\er{MaxVacFull1ninshtrgravortghhghgjkgghklhjgkghghjjkjhjkkggjkhjkhjjhhfhjhkhjjkhjgzzzyyykkknnnNWBWHWPPNjjlkljj}
to anisotropic mediums is the following:
\begin{equation}\label{MaxVacFull1ninshtrgravortghhghgjkgghklhjgkghghjjkjhjkkggjkhjkhjjhhfhjhkhjjkhjgzzzyyykkknnnNWBWHWPPNjjlkljjhjhjhjj}
\vec q=-\mathcal{R}\cdot\nabla_{\vec x}T,
\end{equation}
where $\mathcal{R}\in\mathbb{R}^{3\times 3}$ is a \underline{proper}
matrix valued field, such that
\begin{equation}\label{MaxVacFull1ninshtrgravortghhghgjkgghklhjgkghghjjkjhjkkggjkhjkhjjhhfhjhkhjjkhjgzzzyyykkknnnNWBWHWPPNjjlkljjhjhjhjjhhh}
\vec a\cdot\left(\mathcal{R}\cdot\vec a\right)\geq
0\quad\quad\forall\, \vec a\,\in\,\mathbb{R}^3.
\end{equation}

Next consider the change of certain non-inertial cartesian
coordinate system $(*)$ to another cartesian coordinate system
$(**)$ as:
\begin{equation}\label{noninchredPPNbjjhjjkhkh}
\begin{cases}
\vec x'=A(t)\cdot \vec x+\vec z(t),\\
t'=t,
\end{cases}
\end{equation}
where $A(t)\in SO(3)$ is a rotation. Then by Proposition
\ref{yghgjtgyrtrt} we easily deduce that the Laws in
\er{MaxVacFull1ninshtrgravortghhghgjkgghklhjgkghghjjkjhjkkggjkhjkhjjhhfhjhkjkhbbgjhzzzyyykkknnnNWBWHWPPNjjhjhjhjh},
\er{MaxVacFull1ninshtrgravortghhghgjkgghklhjgkghghjjkjhjkkggjkhjkhjjhhfhjhkjkhbbgjhzzzyyykkknnnNWBWHWPPNjjhjhjhjhkpi},
\er{MaxVacFull1ninshtrgravortghhghgjkgghklhjgkghghjjkjhjkkggjkhjkhjjhhfhjhkjkhbbgjhzzzyyykkknnnNWBWHWPPNjjhjhjhjhkloioj},
\er{MaxVacFull1ninshtrgravortghhghgjkgghklhjgkghghjjkjhjkkggjkhjkhjjhhfhjhkhjjkhjgzzzyyykkknnnNWBWHWPPNjjhhjhj},
\er{MaxVacFull1ninshtrgravortghhghgjkgghklhjgkghghjjkjhjkkggjkhjkhjjhhfhjhkjkhbbgjhzzzyyykkknnnNWBWHWPPNjjhjhjhjhlllkl;kl},
\er{MaxVacFull1ninshtrgravortghhghgjkgghklhjgkghghjjkjhjkkggjkhjkhjjhhfhjhkjkhbbgjhzzzyyykkknnnNWBWHWPPNjjhjhjhjhlllklkljkkhhkgk},
\er{MaxVacFull1ninshtrgravortghhghgjkgghklhjgkghghjjkjhjkkggjkhjkhjjhhfhjhkhjjkhjgzzzyyykkknnnNWBWHWPPNjjlkljjkkkjkjk},
\er{MaxVacFull1ninshtrgravortghhghgjkgghklhjgkghghjjkjhjkkggjkhjkhjjhhfhjhkhjjkhjgzzzyyykkknnnNWBWHWPPNjjlkljj},
\er{MaxVacFull1ninshtrgravortghhghgjkgghklhjgkghghjjkjhjkkggjkhjkhjjhhfhjhkjkhbbgjhzzzyyykkknnnNWBWHWPPNjjhjhjhjhkpihjjgjggjkhhh}
and
\er{MaxVacFull1ninshtrgravortghhghgjkgghklhjgkghghjjkjhjkkggjkhjkhjjhhfhjhkhjjkhjgzzzyyykkknnnNWBWHWPPNjjlkljjhjhjhjjhhh}
are invariant under the change of a non-inertial cartesian
coordinate system given by \er{noninchredPPNbjjhjjkhkh}, provided
that under \er{noninchredPPNbjjhjjkhkh} we have:
\begin{equation}\label{hjgjgjgkllujk}
\begin{cases}\mu'=\mu,\\
E'=E,\\
\mathcal{S}'=\mathcal{S},\\
T'=T,\\
\vec q'=A(t)\cdot\vec q,\\
\mathcal{T}'=A(t)\cdot \mathcal{T}\cdot A^T(t)
\\
p'=p,\\
\chi'=\chi,\\
\mathcal{R}'=A(t)\cdot \mathcal{R}\cdot A^T(t),
\\
\rho'=\rho,\\
\vec u'=A(t)\cdot \vec u+\frac{dA}{dt}(t)\cdot\vec x+\frac{d\vec z}{dt}(t),\\
\vec v'=A(t)\cdot \vec v+\frac{dA}{dt}(t)\cdot\vec x+\frac{d\vec z}{dt}(t),\\
\vec j'=A(t)\cdot \vec j+\rho \frac{dA}{dt}(t)\cdot\vec
x+\rho\frac{d\vec z}{dt}(t).
\end{cases}
\end{equation}

\subsection{Some special cases of continuum mediums}
\subsubsection{The case of an inviscid fluid/gas}
In the case of an inviscid fluid or gas equality
\er{MaxVacFull1ninshtrgravortghhghgjkgghklhjgkghghjjkjhjkkggjkhjkhjjhhfhjhkhjjkhjgzzzyyykkknnnNWBWHWPPNjjhhjhj}
indeed holds. As a consequence, equality
\er{MaxVacFull1ninshtrgravortghhghgjkgghklhjgkghghjjkjhjkkggjkhjkhjjhhfhjhkjkhbbgjhzzzyyykkknnnNWBWHWPPNjjhjhjhjhlllkl;kl}
holds, and moreover, in the case of quasi-reversible process
equality
\er{MaxVacFull1ninshtrgravortghhghgjkgghklhjgkghghjjkjhjkkggjkhjkhjjhhfhjhkjkhbbgjhzzzyyykkknnnNWBWHWPPNjjhjhjhjhlllklkljkkhhkgk}
also holds. Moreover, in the case of a classical ideal gas the
following state equality is well known:
\begin{equation}\label{noninchredPPNbjjhjjkhkhlkkkkl}
p\,=\,\frac{\mu}{m_0}\,k\,T,
\end{equation}
where $m_0$ is the mass of the single molecule of the given gas and
$k$ is the Boltzmann constant. Finally, for the ideal gas, we have
the following expression for the volume density of the internal
energy $E$:
\begin{equation}\label{noninchredPPNbjjhjjkhkhkljjljljkl}
E\,=\,\frac{\mu}{m_0}\,\gamma_0\,k\,T,
\end{equation}
where $\gamma_0>0$ is a constant that depends on the kind of the gas
(for the monatomic gas we have $\gamma_0=\frac{3}{2}$). On the other
hand, in the case of \underline{incompressible} fluid we have
\begin{equation}\label{noninchredPPNbjjhjjkhkhkljjljljklyjhfj}
div_{\vec x}\vec u\,\equiv \,0,
\end{equation}
and the pressure $p$ is unspecified.

Then, as before, we easily deduce that the Laws in
\er{noninchredPPNbjjhjjkhkhlkkkkl},
\er{noninchredPPNbjjhjjkhkhkljjljljkl} and
\er{noninchredPPNbjjhjjkhkhkljjljljklyjhfj} are invariant under the
change of a non-inertial cartesian coordinate system given by
\er{noninchredPPNbjjhjjkhkh}, provided that under
\er{noninchredPPNbjjhjjkhkh} we have \er{hjgjgjgkllujk}.

\subsubsection{The simplest viscous fluid/gas}\label{jghgfhhfhfkhhj}
In the case of the simplest viscous fluid or gas we have the
following equality, that substitutes
\er{MaxVacFull1ninshtrgravortghhghgjkgghklhjgkghghjjkjhjkkggjkhjkhjjhhfhjhkhjjkhjgzzzyyykkknnnNWBWHWPPNjjhhjhj}:
\begin{equation}\label{MaxVacFull1ninshtrgravortghhghgjkgghklhjgkghghjjkjhjkkggjkhjkhjjhhfhjhkhjjkhjgzzzyyykkknnnNWBWHWPPNjjhhjhjkpkl;kl;}
\mathcal{T}=-p\,I+\left(\alpha\left(d_{\vec x}\vec u+\left\{d_{\vec
x}\vec u\right\}^T\right)+\beta\left(div_{\vec x}\vec
u\right)I\right),
\end{equation}
where $\alpha\geq 0$ and $\beta$ are some material coefficients.
Then, as in
\er{MaxVacFull1ninshtrgravortghhghgjkgghklhjgkghghjjkjhjkkggjkhjkhjjhhfhjhkjkhbbgjhzzzyyykkknnnNWBWHWPPNjjhjhjhjhlllkl;kl}
and
\er{MaxVacFull1ninshtrgravortghhghgjkgghklhjgkghghjjkjhjkkggjkhjkhjjhhfhjhkjkhbbgjhzzzyyykkknnnNWBWHWPPNjjhjhjhjhlllklkljkkhhkgk},
by
\er{MaxVacFull1ninshtrgravortghhghgjkgghklhjgkghghjjkjhjkkggjkhjkhjjhhfhjhkhjjkhjgzzzyyykkknnnNWBWHWPPNjjhhjhjkpkl;kl;}
we rewrite the First Law of Thermodynamics
\er{MaxVacFull1ninshtrgravortghhghgjkgghklhjgkghghjjkjhjkkggjkhjkhjjhhfhjhkjkhbbgjhzzzyyykkknnnNWBWHWPPNjjhjhjhjh}
as:
\begin{multline}\label{MaxVacFull1ninshtrgravortghhghgjkgghklhjgkghghjjkjhjkkggjkhjkhjjhhfhjhkjkhbbgjhzzzyyykkknnnNWBWHWPPNjjhjhjhjhlllkl;kljm,n,}
\frac{\partial}{\partial t}\left(\frac{E}{\mu}\right)+\vec
u\cdot\nabla_{\vec x}\left(\frac{E}{\mu}\right)
=\frac{\alpha}{2}\left|d_{\vec x}\vec u+\left\{d_{\vec x}\vec
u\right\}^T\right|^2+\frac{\beta}{2}\left|div_{\vec x}\vec
u\right|^2
\\-p\left(\frac{\partial}{\partial
t}\left(\frac{1}{\mu}\right)+\vec u\cdot\nabla_{\vec
x}\left(\frac{1}{\mu}\right)\right)-\frac{1}{\mu}\Div_{\vec x}\vec
q+\frac{1}{\mu}\left(\vec j-\rho\vec u\right)\cdot\left(\vec
E+\frac{1}{c}\vec u\times\vec B\right),
\end{multline}
and in the case of quasi-reversible process we rewrite
\er{MaxVacFull1ninshtrgravortghhghgjkgghklhjgkghghjjkjhjkkggjkhjkhjjhhfhjhkjkhbbgjhzzzyyykkknnnNWBWHWPPNjjhjhjhjhkloioj}
as:
\begin{multline}\label{MaxVacFull1ninshtrgravortghhghgjkgghklhjgkghghjjkjhjkkggjkhjkhjjhhfhjhkjkhbbgjhzzzyyykkknnnNWBWHWPPNjjhjhjhjhlllklkljkkhhkgklkjjljljl}
\frac{\partial}{\partial t}\left(\frac{E}{\mu}\right)+\vec
u\cdot\nabla_{\vec x}\left(\frac{E}{\mu}\right)
=\frac{\alpha}{2}\left|d_{\vec x}\vec u+\left\{d_{\vec x}\vec
u\right\}^T\right|^2+\frac{\beta}{2}\left|div_{\vec x}\vec
u\right|^2\\-p\left(\frac{\partial}{\partial
t}\left(\frac{1}{\mu}\right)+\vec u\cdot\nabla_{\vec
x}\left(\frac{1}{\mu}\right)\right)+T\left(\frac{\partial}{\partial
t}\left(\frac{\mathcal{S}}{\mu}\right)+\vec u\cdot\nabla_{\vec
x}\left(\frac{\mathcal{S}}{\mu}\right)\right)+\frac{1}{\mu}\left(\vec
j-\rho\vec u\right)\cdot\left(\vec E+\frac{1}{c}\vec u\times\vec
B\right),
\end{multline}
Moreover, in the case of a classical ideal gas the equalities
\er{noninchredPPNbjjhjjkhkhlkkkkl} and
\er{noninchredPPNbjjhjjkhkhkljjljljkl} are still valid in the case
of viscous flow. On the other hand, in the case of
\underline{incompressible} viscous fluid we have
\er{noninchredPPNbjjhjjkhkhkljjljljklyjhfj} and the pressure $p$ is
also unspecified.

Then, as before, by Proposition \ref{yghgjtgyrtrt} we easily deduce
that the Laws in
\er{MaxVacFull1ninshtrgravortghhghgjkgghklhjgkghghjjkjhjkkggjkhjkhjjhhfhjhkhjjkhjgzzzyyykkknnnNWBWHWPPNjjhhjhjkpkl;kl;},
\er{MaxVacFull1ninshtrgravortghhghgjkgghklhjgkghghjjkjhjkkggjkhjkhjjhhfhjhkjkhbbgjhzzzyyykkknnnNWBWHWPPNjjhjhjhjhlllkl;kljm,n,},
\er{MaxVacFull1ninshtrgravortghhghgjkgghklhjgkghghjjkjhjkkggjkhjkhjjhhfhjhkjkhbbgjhzzzyyykkknnnNWBWHWPPNjjhjhjhjhlllklkljkkhhkgklkjjljljl},
\er{noninchredPPNbjjhjjkhkhlkkkkl},
\er{noninchredPPNbjjhjjkhkhkljjljljkl} and
\er{noninchredPPNbjjhjjkhkhkljjljljklyjhfj} are invariant under the
change of a non-inertial cartesian coordinate system given by
\er{noninchredPPNbjjhjjkhkh}, provided that under
\er{noninchredPPNbjjhjjkhkh} we have \er{hjgjgjgkllujk}.

Finally, in the case of an anisotropic Newtonian fluid we have the
following generalization of the viscosity law
\er{MaxVacFull1ninshtrgravortghhghgjkgghklhjgkghghjjkjhjkkggjkhjkhjjhhfhjhkhjjkhjgzzzyyykkknnnNWBWHWPPNjjhhjhjkpkl;kl;}:
\begin{equation}\label{MaxVacFull1ninshtrgravortghhghgjkgghklhjgkghghjjkjhjkkggjkhjkhjjhhfhjhkhjjkhjgzzzyyykkknnnNWBWHWPPNjjhhjhjkpkl;kl;aa}
\mathcal{T}=-p\,I+\Theta\cdot\left(d_{\vec x}\vec u+\left\{d_{\vec
x}\vec u\right\}^T\right),
\end{equation}
where $\Theta:=\Theta(\vec x,t)$ is a linear operator ($9\times 9$
matrix) from $\mathbb{R}^{3\times 3}$ to $\mathbb{R}^{3\times 3}$,
such that
$$(\Theta\cdot G)^T=\Theta\cdot G\quad\quad\forall \,G\in \mathbb{R}^{3\times 3}\quad\text{such that}\quad G^T=G.$$
Then, using Lemma \ref{khhjgjgggh}, by \er{hjgjgjgkllujk} we easily
deduce that the law in
\er{MaxVacFull1ninshtrgravortghhghgjkgghklhjgkghghjjkjhjkkggjkhjkhjjhhfhjhkhjjkhjgzzzyyykkknnnNWBWHWPPNjjhhjhjkpkl;kl;aa}
is invariant under the change of a non-inertial cartesian coordinate
system given by \er{noninchredPPNbjjhjjkhkh}, provided that, under
\er{noninchredPPNbjjhjjkhkh} we have:
\begin{equation}\label{hjgjgjgkllujkjkhkkkjklaassaa}
\begin{cases}
\mathcal{T}'=A(t)\cdot \mathcal{T}\cdot A^T(t),\\
\Theta'=\left(A(t)\otimes A(t)\right)\cdot\Theta\cdot
\left(A^T(t)\otimes A^T(t)\right),
\end{cases}
\end{equation}
where the sign $\otimes$ in \er{hjgjgjgkllujkjkhkkkjklaassaa} means
the tensor product of the matrices, i.e. for given two linear
operators (matrices) $A\in\mathbb{R}^{3\times 3}$ and $B\in
\mathbb{R}^{3\times 3}$
their tensor product
$A\otimes B$ is a linear operator from $\mathbb{R}^{3\times 3}$ to
$\mathbb{R}^{3\times 3}$, defined by the identity:
\begin{equation}\label{fyfgjguyiuiHHHHHHkjjjjlzzaassaa}
(A\otimes B)\cdot(\vec a\otimes \vec b)=(A\cdot \vec a)\otimes
(B\cdot \vec b)\quad\quad\forall \vec a\in\mathbb{R}^{3},\;\forall
\vec b\in\mathbb{R}^{3}.
\end{equation}

\subsection{Lagrangian coordinates and the simplest models of elastic
bodies}\label{hjgyujtgjtututtu} In some cartesian coordinate system
consider a motion of some continuum medium occupying a region
$\Omega\subset\mathbb{R}^3$ at some fixed instant of time $t=t_0$
and having the velocity field $\vec u:=\vec u(\vec x,t)$. Next let
$\vec r(t,\vec y):\mathbb{R}\times\Omega\to\mathbb{R}^3$ be a
solution of the following initial value problem for an ordinary
differential equation:
\begin{equation}\label{hjgghhghhffhkk}
\begin{cases}
\frac{\partial\vec r}{\partial t}(t,\vec y)=\vec u\left(\vec
r(t,\vec y),t\right)\quad\quad\forall t\in\mathbb{R},\;\;\forall\vec
y\in\Omega\\
\vec r(t_0,\vec y)=\vec y\quad\quad\forall\vec y\in\Omega.
\end{cases}
\end{equation}
Then, clearly $\vec x=\vec r(t,\vec y)$ stands for the spatial
coordinates at the instant of time $t$ of the parcel of continuum,
having initial coordinates $\vec y$. Moreover, by
\er{hjgghhghhffhkk} the matrix-valued function $d_{\vec y}\vec
r(t,\vec y)$ satisfies the following initial value problem for a
matrix ordinary differential equation:
\begin{equation}\label{hjgghhghhffhkkihiyjh}
\begin{cases}
\frac{\partial}{\partial t}\left(d_{\vec y}\vec r(t,\vec
y)\right)=\left(d_{\vec x}\vec u\left(\vec r(t,\vec
y),t\right)\right)\cdot d_{\vec y}\vec r(t,\vec y)\quad\quad\forall
t\in\mathbb{R},\;\;\forall\vec
y\in\Omega\\
d_{\vec y}\vec r(t_0,\vec y)=I\quad\quad\vec y\in\Omega.
\end{cases}
\end{equation}
In particular, using the Liouville Theorem in the theory of linear
ordinary differential systems by \er{hjgghhghhffhkkihiyjh} we deduce
that $det\left\{d_{\vec y}\vec r(t,\vec y)\right\}$ satisfies
\begin{equation}\label{hjgghhghhffhkkihiyjhkpkpkbb}
det\left\{d_{\vec y}\vec r(t,\vec
y)\right\}=e^{\int_{t_0}^{t}\left(div_{\vec x}\vec u\left(\vec
r(\tau,\vec y),\tau\right)\right)d\tau}\quad\quad\forall
t\in\mathbb{R},\;\;\forall\vec y\in\Omega.
\end{equation}
Thus
by \er{hjgghhghhffhkkihiyjhkpkpkbb} we deduce that
$det\left\{d_{\vec y}\vec r(t,\vec y)\right\}\neq 0$ for every
instant of time $t$ and so, for the given instant of time $t$ the
mapping $\vec x=\vec r(t,\vec y)$ is locally invertible i.e. the
equation $\vec x=\vec r(t,\vec y)$ can be resolved in $\vec y$. Thus
there exists a regular mapping $\vec y=\vec f(\vec x,t)$, such that
\begin{equation}\label{hjgghhghhffhkkihiyjhjkhjhjh}
\vec f\left(\vec r(t,\vec y),t\right)=\vec y\quad\forall \vec
y\in\Omega\quad\quad\text{and}\quad\quad\vec r\left(t,\vec f(\vec
x,t)\right)=\vec x.
\end{equation}
Then clearly $\vec y=\vec f(\vec x,t)$ stands for the initial
coordinates at the time $t_0$ of the parcel of continuum, having
coordinates $\vec x$ at the instant of time $t$. Next
differentiating the first equation in
\er{hjgghhghhffhkkihiyjhjkhjhjh} by $t$ due to the chain rule and
using \er{hjgghhghhffhkk} we deduce:
\begin{equation}\label{hjgghhghhffhkkihiyjhjkhjhjhojujjkh}
0=\frac{\partial\vec f}{\partial t}\left(\vec r(t,\vec
y),t\right)+d_{\vec x}\vec f\left(\vec r(t,\vec
y),t\right)\cdot\frac{\partial\vec r}{\partial t}(t,\vec
y)=\frac{\partial\vec f}{\partial t}\left(\vec r(t,\vec
y),t\right)+d_{\vec x}\vec f\left(\vec r(t,\vec y),t\right)\cdot\vec
u\left(\vec r(t,\vec y),t\right)\quad\forall \vec y\in\Omega.
\end{equation}
So $\vec f(\vec x,t)$ satisfies:
\begin{equation}\label{hjgghhghhffhkkmkljjnnnn}
\begin{cases}
\frac{\partial\vec f}{\partial t}\left(\vec x,t\right)+d_{\vec
x}\vec f\left(\vec x,t\right)\cdot\vec u\left(\vec x,t\right)=0
\\
\vec f(\vec x,t_0)=\vec x.
\end{cases}
\end{equation}
The \underline{cartesian} coordinates $(\vec x,t)$ are called the
Eulerian coordinates of the continuum medium. In contrast, change of
variables
\begin{equation}\label{hjgghhghhffhkkmkljjnnnnoujku}
\begin{cases}
t=t\\
\vec y=\vec f\left(\vec x,t\right)
\end{cases}
\end{equation}
to the new coordinates $(\vec y,t)$ of the space-time leads to
generally non-cartesian \underline{curvilinear} coordinates that
called the Lagrangian or the reference coordinates of the continuum
medium. Next note that in the case of a given scalar field in the
Eulerian coordinates $\Theta=\Theta(\vec x,t)$ and the corresponding
field in Lagrangian coordinates $\Theta_1(\vec
y,t):=\Theta\left(\vec r(t,\vec y),t\right)$, due to the chain rule
and using \er{hjgghhghhffhkk} we have:
\begin{multline}\label{hjgghhghhffhkkihiyjhjkhjhjhojujjkhjjklkjkjgbnvv}
\frac{\partial\Theta_1}{\partial t}(\vec
y,t)=\frac{\partial\Theta}{\partial t}\left(\vec r(t,\vec
y),t\right)+\nabla_{\vec x}\Theta\left(\vec r(t,\vec
y),t\right)\cdot\frac{\partial\vec r}{\partial t}(t,\vec
y)\\=\frac{\partial\Theta}{\partial t}\left(\vec r(t,\vec
y),t\right)+\vec u\left(\vec r(t,\vec y),t\right)\cdot\nabla_{\vec
x}\Theta\left(\vec r(t,\vec y),t\right)\quad\forall \vec y\in\Omega.
\end{multline}
So for  $\frac{\partial\Theta_1}{\partial t}(\vec y,t)$ in the
Lagrangian coordinates corresponds the following expression in the
Eulerian coordinates:
\begin{equation}\label{hjgghhghhffhkkmkljjnnnnoujkujjjjljj}
\frac{\partial\Theta_1}{\partial t}(\vec
y,t)=\frac{\partial\Theta}{\partial t}\left(\vec x,t\right)+\vec
u\left(\vec x,t\right)\cdot\nabla_{\vec x}\Theta\left(\vec
x,t\right).
\end{equation}
Thus by \er{hjgghhghhffhkkmkljjnnnnoujkujjjjljj} we also obtain that
in the case of a given vector field in the Eulerian coordinates
$\vec g=\vec g(\vec x,t)$ and the corresponding field in Lagrangian
coordinates $\vec g_1(\vec y,t):=\vec g\left(\vec x,t\right)=\vec
g\left(\vec r(t,\vec y),t\right)$ for $\frac{\partial\vec
g_1}{\partial t}(\vec y,t)$ in the Lagrangian coordinates
corresponds the following expression in the Eulerian coordinates:
\begin{equation}\label{hjgghhghhffhkkmkljjnnnnoujkujjjjljjkhhjkjk}
\frac{\partial\vec g_1}{\partial t}(\vec y,t)=\frac{\partial\vec
g}{\partial t}\left(\vec x,t\right)+ d_{\vec x}\vec g\left(\vec
x,t\right)\cdot \vec u\left(\vec x,t\right).
\end{equation}

Next, as before, assume that the change of some non-inertial
cartesian system $(*)$ of Eulerian coordinates to another cartesian
system $(**)$ of Eulerian coordinates is of the form:
\begin{equation}\label{jjjjkkkk}
\begin{cases}
\vec x'=A(t)\cdot\vec x+\vec z(t),\\
t'=t,
\end{cases}
\end{equation}
where $A(t)\in SO(3)$. We would like to derive that the law of
transformation of the Lagrangian coordinates $(\vec y,t)\to(\vec
y',t')$, consistent with \er{jjjjkkkk}, is the following:
\begin{equation}\label{jjjjkkkkkhhhhjhjh}
\begin{cases}
\vec y'=A(t_0)\cdot\vec y+\vec z(t_0),\\
t'=t,
\end{cases}
\end{equation}
Indeed consistently with \er{jjjjkkkk} we have
\begin{equation}\label{hjgjgjgkllujkjjj}
\vec u'(\vec x',t)=A(t)\cdot \vec u(\vec
x,t)+\frac{dA}{dt}(t)\cdot\vec x+\frac{d\vec z}{dt}(t).
\end{equation}
Thus if we define $\vec r'(t',\vec
y'):\mathbb{R}\times\mathbb{R}^3\to\mathbb{R}^3$ as:
\begin{equation}\label{hjgjgjgkllujkjjjkhhjh}
\vec r'(t',\vec y')=\vec r'\left(t,A(t_0)\cdot\vec y+\vec
z(t_0)\right):=A(t)\cdot\vec r(t,\vec y)+\vec z(t),
\end{equation}
then by \er{hjgghhghhffhkk} we obtain firstly
\begin{equation}\label{hjgjgjgkllujkjjjkhhjhjhhhj}
\vec r'(t_0,\vec y')=\vec r'\left(t_0,A(t_0)\cdot\vec y+\vec
z(t_0)\right)=A(t_0)\cdot\vec r(t_0,\vec y)+\vec
z(t_0)=A(t_0)\cdot\vec y+\vec z(t_0)=\vec y',
\end{equation}
and secondly with the help of \er{hjgjgjgkllujkjjjkhhjh} and
\er{hjgjgjgkllujkjjj}:
\begin{multline}\label{hjgjgjgkllujkjjjkhhjhklkljkl}
\frac{\partial\vec r'}{\partial t'}(t',\vec y')=\frac{\partial\vec
r'}{\partial t}\left(t,A(t_0)\cdot\vec y+\vec
z(t_0)\right)=A(t)\cdot \frac{\partial\vec r}{\partial
t}\left(t,\vec y\right)+\frac{dA}{dt}(t)\cdot\vec r\left(t,\vec
y\right)+\frac{d\vec z}{dt}(t)\\= A(t)\cdot \vec u\left(\vec
r(t,\vec y),t\right)+\frac{dA}{dt}(t)\cdot\vec r\left(t,\vec
y\right)+\frac{d\vec z}{dt}(t)=\vec u'\left(\vec r'(t',\vec
y'),t'\right).
\end{multline}
So, by \er{hjgjgjgkllujkjjjkhhjhjhhhj} and
\er{hjgjgjgkllujkjjjkhhjhklkljkl} consistently with
\er{hjgghhghhffhkk} we have
\begin{equation}\label{hjgghhghhffhkklkjljjkjkkjl}
\begin{cases}
\frac{\partial\vec r'}{\partial t'}(t',\vec y')=\vec u'\left(\vec
r'(t',\vec y'),t'\right)\\
\vec r'(t_0,\vec y')=\vec y'.
\end{cases}
\end{equation}
On the other hand, by \er{hjgjgjgkllujkjjjkhhjh} for inverse
mappings we have
\begin{equation}\label{hjgjgjgkllujkjjjkhhjhljjk}
\vec f'(\vec x',t')=\vec f'\left(A(t)\cdot\vec x+\vec
z(t),t\right)=A(t_0)\cdot\vec f(\vec x,t)+\vec z(t_0).
\end{equation}

\subsubsection{Strain tensor and the Hooke's
law}\label{hjgujhghfhffgfg} Next, it is well known that the finite
strain tensor $\mathcal{E}(\vec x,t)\in\mathbb{R}^{3\times 3}$ of an
elastic continuum medium in the cartesian Eulerian coordinates
$(\vec x,t)$ has the form:
\begin{equation}\label{hjgjgjgkllujkjjjkhhjhljjkjkljjljj}
\mathcal{E}(\vec x,t):=\frac{1}{2}\left(I-\left\{d_{\vec x}\vec
f(\vec x,t)\right\}^T\cdot d_{\vec x}\vec f(\vec x,t)\right),
\end{equation}
where the mapping $\vec f(\vec x,t)$ is given by
\er{hjgghhghhffhkkihiyjhjkhjhjh} and satisfies
\er{hjgghhghhffhkkmkljjnnnn}. Then by \er{jjjjkkkk} and
\er{hjgjgjgkllujkjjjkhhjhljjk} we deduce that under \er{jjjjkkkk}
the matrix valued field $\mathcal{E}$ transforms as:
\begin{equation}\label{hjgjgjgkllujkjjjkhhjhljjkjkljjljjlljjkhhj}
\mathcal{E}'(\vec x',t)=A(t)\cdot\mathcal{E}(\vec x,t)\cdot A^T(t),
\end{equation}
i.e. $\mathcal{E}$ is a proper matrix field as it was defined in
Definition \ref{bggghghgj}.
\begin{remark}\label{uggghfhfhfg}
It is quite clear that the finite strain tensor, defined by
\er{hjgjgjgkllujkjjjkhhjhljjkjkljjljj} depends essentially on the
choice of initial instant of time $t_0$. Thus, the initial time
$t_0$ for \er{hjgjgjgkllujkjjjkhhjhljjkjkljjljj} is always chosen
(possibly fictitiously) in such a way that
our elastic body is relaxed at time $t_0$.
\end{remark}
Next, in the case of the simplest elastic body we have the following
Hooke's law that is similar to
\er{MaxVacFull1ninshtrgravortghhghgjkgghklhjgkghghjjkjhjkkggjkhjkhjjhhfhjhkhjjkhjgzzzyyykkknnnNWBWHWPPNjjhhjhjkpkl;kl;}:
\begin{equation}\label{MaxVacFull1ninshtrgravortghhghgjkgghklhjgkghghjjkjhjkkggjkhjkhjjhhfhjhkhjjkhjgzzzyyykkknnnNWBWHWPPNjjhhjhjkphkhkhhklklkl}
\mathcal{T}(\vec x,t)=\left(\alpha\mathcal{E}(\vec
x,t)+\beta\left(tr\left\{\mathcal{E}(\vec
x,t)\right\}\right)I\right),
\end{equation}
where $\alpha$ and $\beta$ are some material coefficients, which are
not necessary constant. Here $\mathcal{T}(\vec x,t)$ is the Cauchy
stress tensor appearing in the equations of the motion of the medium
\er{MaxVacFull1ninshtrgravortghhghgjkgghklhjgkghghjjkjhjkkggjkhjkhjjhhfhjhkjkhbbgjhzzzyyykkknnnNWBWHWPPNjj}
in the Eulerian coordinates $(\vec x,t)$ and $\mathcal{E}(\vec x,t)$
is the finite strain tensor defined by
\er{hjgjgjgkllujkjjjkhhjhljjkjkljjljj}. Then by
\er{hjgjgjgkllujkjjjkhhjhljjkjkljjljjlljjkhhj} we easily deduce that
the Hooke's law
\er{MaxVacFull1ninshtrgravortghhghgjkgghklhjgkghghjjkjhjkkggjkhjkhjjhhfhjhkhjjkhjgzzzyyykkknnnNWBWHWPPNjjhhjhjkphkhkhhklklkl}
is invariant under the change of a non-inertial cartesian coordinate
system given by \er{jjjjkkkk}, provided that, as usual, under
\er{jjjjkkkk} we have:
\begin{equation}\label{hjgjgjgkllujkjkhkkkjkl}
\begin{cases}
\alpha'=\alpha,\\
\beta'=\beta,\\
\mathcal{T}'=A(t)\cdot \mathcal{T}\cdot A^T(t).
\end{cases}
\end{equation}

Finally, in the case of the anisotropic elastic body we have the
following generalization of the Hooke's law
\er{MaxVacFull1ninshtrgravortghhghgjkgghklhjgkghghjjkjhjkkggjkhjkhjjhhfhjhkhjjkhjgzzzyyykkknnnNWBWHWPPNjjhhjhjkphkhkhhklklkl}:
\begin{equation}\label{MaxVacFull1ninshtrgravortghhghgjkgghklhjgkghghjjkjhjkkggjkhjkhjjhhfhjhkhjjkhjgzzzyyykkknnnNWBWHWPPNjjhhjhjkphkhkhhklklklaa}
\mathcal{T}(\vec x,t)=\Lambda(\vec x,t)\cdot\mathcal{E}(\vec x,t),
\end{equation}
where $\Lambda:=\Lambda(\vec x,t)$ is a linear operator ($9\times 9$
matrix) from $\mathbb{R}^{3\times 3}$ to $\mathbb{R}^{3\times 3}$,
such that
$$(\Lambda\cdot G)^T=\Lambda\cdot G\quad\quad\forall \,G\in \mathbb{R}^{3\times 3}\quad\text{such that}\quad G^T=G.$$
Then, as before, using Lemma \ref{khhjgjgggh}, by
\er{hjgjgjgkllujkjjjkhhjhljjkjkljjljjlljjkhhj} we easily deduce that
the law in
\er{MaxVacFull1ninshtrgravortghhghgjkgghklhjgkghghjjkjhjkkggjkhjkhjjhhfhjhkhjjkhjgzzzyyykkknnnNWBWHWPPNjjhhjhjkphkhkhhklklklaa}
is invariant under the change of a non-inertial cartesian coordinate
system given by \er{jjjjkkkk}, provided that, under \er{jjjjkkkk} we
have:
\begin{equation}\label{hjgjgjgkllujkjkhkkkjklaa}
\begin{cases}
\mathcal{T}'=A(t)\cdot \mathcal{T}\cdot A^T(t),\\
\Lambda'=\left(A(t)\otimes A(t)\right)\cdot\Lambda\cdot
\left(A^T(t)\otimes A^T(t)\right),
\end{cases}
\end{equation}
where the sign $\otimes$ in \er{hjgjgjgkllujkjkhkkkjklaa} means the
tensor product of the matrices, i.e. for given two linear operators
(matrices) $A\in\mathbb{R}^{3\times 3}$ and $B\in
\mathbb{R}^{3\times 3}$
their tensor product
$A\otimes B$ is a linear operator from $\mathbb{R}^{3\times 3}$ to
$\mathbb{R}^{3\times 3}$, defined by the identity:
\begin{equation}\label{fyfgjguyiuiHHHHHHkjjjjlzzaa}
(A\otimes B)\cdot(\vec a\otimes \vec b)=(A\cdot \vec a)\otimes
(B\cdot \vec b)\quad\quad\forall \vec a\in\mathbb{R}^{3},\;\forall
\vec b\in\mathbb{R}^{3}.
\end{equation}

\subsubsection{The equations of the motion of the general continuum medium in Lagrangian coordinates}
The equation of the motion of the continuum in the cartesian
Eulerian coordinates $(\vec x,t)$ has the form
\er{MaxVacFull1ninshtrgravortghhghgjkgghklhjgkghghjjkjhjkkggjkhjkhjjhhfhjhkjkhbbgjhzzzyyykkknnnNWBWHWPPNjj},
i.e:
\begin{multline}\label{MaxVacFull1ninshtrgravortghhghgjkgghklhjgkghghjjkjhjkkggjkhjkhjjhhfhjhkjkhbbgjhzzzyyykkknnnNWBWHWPPNjjjljkjk}
\mu\left(\frac{\partial\vec u}{\partial t}+d_{\vec x}\vec u\cdot
\vec u\right)=-\mu\vec u\times curl_{\vec x}\vec
v+\mu\left(\partial_{t}\vec
v+\nabla_{\vec x}\frac{1}{2}|\vec v|^2\right)+\rho\vec E+\frac{1}{c}\vec j\times\vec B+\Div_{\vec x}\mathcal{T}=\\
=-\mu(\vec u-\vec v)\times curl_{\vec x}\vec v
+\mu\left(\partial_t\vec v+ d_{\vec x}\vec v\cdot\vec
v\right)+\rho\vec E+\frac{1}{c}\vec j\times\vec B+\Div_{\vec
x}\mathcal{T},
\end{multline}
and the continuum equation has the form
\er{MaxVacFull1ninshtrgravortghhghgjkgghklhjgkghghjjkjhjkkggjkhjkhjjhhfhjhkhjjkhjgzzzyyykkknnnNWBWHWPPNjj},
i.e:
\begin{equation}\label{MaxVacFull1ninshtrgravortghhghgjkgghklhjgkghghjjkjhjkkggjkhjkhjjhhfhjhkhjjkhjgzzzyyykkknnnNWBWHWPPNjjjkk}
\frac{\partial\mu}{\partial t}+div_{\vec x}\left(\mu\vec u\right)=0.
\end{equation}
We would like to get the analogues of these equations in the
Lagrangian coordinates $(\vec y,t)$. First of all observe that since
the inverse change of coordinates is given by the formula:
\begin{equation}\label{hjhjjttuttutulj}
\begin{cases}
t=t\\
\vec x=\vec r(t,\vec y),
\end{cases}
\end{equation}
then clearly by the formula of the Change of Variable of the
Integration we deduce:
\begin{equation}\label{hjhjjttuttutukjhkhlk}
\mu(\vec y,t_0)=\mu\left(\vec r(t,\vec y),t\right)\left|det\{d_{\vec
y}\vec r(t,\vec y)\}\right|,
\end{equation}
i.e:
\begin{equation}\label{hjhjjttuttutukjhkhlkkjkh}
\mu\left(\vec r(t,\vec y),t\right)=\mu(\vec y,t_0)\left|det\{d_{\vec
y}\vec r(t,\vec y)\}\right|^{-1}\quad\quad\forall t,\;\forall\vec
y\in\Omega.
\end{equation}
In particular, by \er{hjgghhghhffhkkihiyjhjkhjhjh},
\er{hjhjjttuttutukjhkhlkkjkh} in Eulerian coordinates reeds as
\begin{equation}\label{hjhjjttuttutukjhkhlkkjkhjljjzz}
\mu\left(\vec x,t\right)=\mu(\vec x,t_0)\left|det\{d_{\vec x}\vec
f(\vec x,t)\}\right|.
\end{equation}
On the other hand, consistently with
\er{hjgghhghhffhkkmkljjnnnnoujkujjjjljj}, the analog of continuum
equation
\er{MaxVacFull1ninshtrgravortghhghgjkgghklhjgkghghjjkjhjkkggjkhjkhjjhhfhjhkhjjkhjgzzzyyykkknnnNWBWHWPPNjjjkk}
in Lagrangian coordinates $(\vec y,t)$ is the following
\begin{equation}\label{MaxVacFull1ninshtrgravortghhghgjkgghklhjgkghghjjkjhjkkggjkhjkhjjhhfhjhkhjjkhjgzzzyyykkknnnNWBWHWPPNjjjkkkyiyyyy}
\frac{\partial}{\partial t}\left(\mu\left(\vec r(t,\vec
y),t\right)\right)+\left(div_{\vec x}\vec u\left(\vec r(t,\vec
y),t\right)\right)\left(\mu\left(\vec r(t,\vec y),t\right)\right)=0.
\end{equation}
However, we can easily deduce that equality
\er{hjgghhghhffhkkihiyjhkpkpkbb} together with
\er{hjhjjttuttutukjhkhlkkjkh} implies
\er{MaxVacFull1ninshtrgravortghhghgjkgghklhjgkghghjjkjhjkkggjkhjkhjjhhfhjhkhjjkhjgzzzyyykkknnnNWBWHWPPNjjjkkkyiyyyy}.
Thus equality \er{hjhjjttuttutukjhkhlkkjkh} indeed substitutes the
continuum equation
\er{MaxVacFull1ninshtrgravortghhghgjkgghklhjgkghghjjkjhjkkggjkhjkhjjhhfhjhkhjjkhjgzzzyyykkknnnNWBWHWPPNjjjkk}
in the Lagrangian coordinates $(\vec y,t)$.

Next, the following Calculus fact is well known in Continuum
Mechanics: if, given some matrix valued field
\begin{equation}\label{hkhkhhhnnmhhhjklkl}
\mathcal{T}(\vec x,t)\in\mathbb{R}^{3\times 3},
\end{equation}
we denote
\begin{equation}\label{Mgjgjghghklkl}
F(\vec y,t):=d_{\vec y}\vec r(t,\vec y)\in\mathbb{R}^{3\times 3},
\end{equation}
and
\begin{equation}\label{Mgjgjghghkjljlljljjjkjhhhjh}
\mathcal{R}(\vec y,t):=\mathcal{T}\left(\vec r(t,\vec
y),t\right)\cdot \{F^{-1}(\vec y,t)\}^T\left(det\,F(\vec
y,t)\right),
\end{equation}
then we must have:
\begin{equation}\label{hjkgugfhfhffjkjkojljhjjgjjjggj}
\Div_{\vec y}\mathcal{R}(\vec y,t)=\left( \Div_{\vec
x}\mathcal{T}\left(\vec r(t,\vec y),t\right)\right)\left(det\,F(\vec
y,t)\right).
\end{equation}
In particular, if $\mathcal{T}(\vec x,t)$ is the Cauchy stress
tensor, then $\mathcal{R}(\vec y,t)$ defined by
\er{Mgjgjghghkjljlljljjjkjhhhjh} is called the first
Piola--Kirchhoff stress tensor. Then by
\er{hjkgugfhfhffjkjkojljhjjgjjjggj}, \er{hjhjjttuttutukjhkhlkkjkh},
\er{hjgghhghhffhkk} and
\er{hjgghhghhffhkkmkljjnnnnoujkujjjjljjkhhjkjk} we finally rewrite
\er{MaxVacFull1ninshtrgravortghhghgjkgghklhjgkghghjjkjhjkkggjkhjkhjjhhfhjhkjkhbbgjhzzzyyykkknnnNWBWHWPPNjjjljkjk}
in Lagrangian coordinates as:
\begin{multline}\label{MaxVacFull1ninshtrgravortghhghgjkgghklhjhgjggjgkghghjjkjhjkkggjkhjkhjjhhfhjhkjkhbbgjhzzzyyykkknnnNWBWHWPPNjjjljkjk}
\mu(\vec y,t_0)\frac{\partial^2\vec r}{\partial t^2}(t,\vec
y)=\Div_{\vec y}\mathcal{R}(\vec y,t)\\-\mu(\vec
y,t_0)\frac{\partial\vec r}{\partial t}(t,\vec y)\times curl_{\vec
x}\vec v\left(\vec r(t,\vec y),t\right)+\mu(\vec
y,t_0)\left(\partial_{t}\vec v\left(\vec r(t,\vec
y),t\right)+\nabla_{\vec x}\frac{1}{2}\left|\vec v\left(\vec
r(t,\vec y),t\right)\right|^2\right)\\+\frac{\mu(\vec
y,t_0)}{\mu\left(\vec r(t,\vec y),t\right)}\left(\rho\left(\vec
r(t,\vec y),t\right)\vec E\left(\vec r(t,\vec
y),t\right)+\frac{1}{c}\vec j\left(\vec r(t,\vec
y),t\right)\times\vec B\left(\vec r(t,\vec y),t\right)\right).
\end{multline}

\subsection{Propagation of sound wave in the moving inviscid gas or fluid}\label{jihggghjgjgg}
Again, consistently with
\er{MaxVacFull1ninshtrgravortghhghgjkgghklhjgkghghjjkjhjkkggjkhjkhjjhhfhjhkjkhbbgjhzzzyyykkknnnNWBWHWPPNjj},
consider in some cartesian coordinate system $(*)$ the second Law of
Newton for the moving continuum medium with the inertial mass
density $\mu$, the field of average (macroscopic) velocities $\vec
u$, the charge density $\rho$ and the electric current density $\vec
j$:
\begin{multline}\label{MaxVacFull1ninshtrgravortghhghgjkgghklhjgkghghjjkjhjkkggjkhjkhjjhhfhjhkjkhbbgjhzzzyyykkknnnNWBWHWPPNjj1}
\mu\left(\frac{\partial\vec u}{\partial t}+d_{\vec x}\vec u\cdot
\vec u\right)=\frac{\partial}{\partial t}(\mu\vec u)+\Div_{\vec
x}\left\{\mu\vec u\otimes\vec u\right\}=\\-\mu\vec u\times
curl_{\vec x}\vec v+\mu\left(\partial_{t}\vec
v+\nabla_{\vec x}\frac{1}{2}|\vec v|^2\right)+\rho\vec E+\frac{1}{c}\vec j\times\vec B+\Div_{\vec x}\mathcal{T}=\\
=-\mu(\vec u-\vec v)\times curl_{\vec x}\vec v
+\mu\left(\partial_t\vec v+ d_{\vec x}\vec v\cdot\vec
v\right)+\rho\vec E+\frac{1}{c}\vec j\times\vec B+\Div_{\vec
x}\mathcal{T}.
\end{multline}
Here $\rho\vec E+\frac{1}{c}\vec j\times\vec B$ is the volume
density of the Lorentz force where $\vec E$ and $\vec B$ are outer
electric and magnetic fields, assumed to be changing smoothly and
almost constant in the microscopic level, $\vec v$ is a vectorial
gravitational potential also assumed to be changing smoothly and
almost constant in the microscopic level, and $\mathcal{T}\in
\R^{3\times 3}$ is the symmetric Cauchy stress tensor of the
continuum medium. Moreover, the mass density $\mu$, clearly
satisfies the continuum equation
\er{MaxVacFull1ninshtrgravortghhghgjkgghklhjgkghghjjkjhjkkggjkhjkhjjhhfhjhkhjjkhjgzzzyyykkknnnNWBWHWPPNjj}:
\begin{equation}\label{MaxVacFull1ninshtrgravortghhghgjkgghklhjgkghghjjkjhjkkggjkhjkhjjhhfhjhkhjjkhjgzzzyyykkknnnNWBWHWPPNjj1}
\frac{\partial\mu}{\partial t}+div_{\vec x}\left(\mu\vec u\right)=0.
\end{equation}
Next, by
\er{MaxVacFull1ninshtrgravortghhghgjkgghklhjgkghghjjkjhjkkggjkhjkhjjhhfhjhkjkhbbgjhzzzyyykkknnnNWBWHWPPNjjhjhjhjh}
the First Law of Thermodynamics of this moving medium has the
following form:
\begin{multline}\label{MaxVacFull1ninshtrgravortghhghgjkgghklhjgkghghjjkjhjkkggjkhjkhjjhhfhjhkjkhbbgjhzzzyyykkknnnNWBWHWPPNjjhjhjhjh1}
\frac{\partial E}{\partial t}+\Div_{\vec x}\left\{E\vec
u\right\}=\mu\left(\frac{\partial}{\partial
t}\left(\frac{E}{\mu}\right)+\vec u\cdot\nabla_{\vec
x}\left(\frac{E}{\mu}\right)\right)\\ =\frac{1}{2}\left(d_{\vec
x}\vec u+\left\{d_{\vec x}\vec
u\right\}^T\right):\mathcal{T}-\Div_{\vec x}\vec q+\left(\vec
j-\rho\vec u\right)\cdot\left(\vec E+\frac{1}{c}\vec u\times\vec
B\right).
\end{multline}
Here $E$ is the volume density of the internal energy (energy per
unit volume) and consistently $\frac{E}{\mu}$ is the internal energy
per unit mass and $\vec q$ is the heat flux. Next, in the case of
inviscid fluid or gas
\er{MaxVacFull1ninshtrgravortghhghgjkgghklhjgkghghjjkjhjkkggjkhjkhjjhhfhjhkhjjkhjgzzzyyykkknnnNWBWHWPPNjjhhjhj}
holds:
\begin{equation}\label{MaxVacFull1ninshtrgra,mnvortghhghgjkgghklhjgkghghjjkjhjkkggjkhjkhjjhhfhjhkjkhbbgjhzzzyyykkknnnNWBWHWPPNjjhjhjhjhplllhjggjl}
\mathcal{T}=-p I.
\end{equation}
Thus, as before, we rewrite,
\er{MaxVacFull1ninshtrgravortghhghgjkgghklhjgkghghjjkjhjkkggjkhjkhjjhhfhjhkjkhbbgjhzzzyyykkknnnNWBWHWPPNjj1}
as:
\begin{multline}\label{MaxVacFull1ninshtrgravortghhghgjkgghklhjgkghghjjkjhjkkggjkhjkhjjhhfhjhkjkhbbgjhzzzyyykkknnnNWBWHWPPNjjhuhh}
\mu\left(\frac{\partial\vec u}{\partial t}+d_{\vec x}\vec u\cdot
\vec u\right)=\frac{\partial}{\partial t}(\mu\vec u)+\Div_{\vec
x}\left\{\mu\vec u\otimes\vec u\right\}=\\-\mu\vec u\times
curl_{\vec x}\vec v+\mu\left(\partial_{t}\vec
v+\nabla_{\vec x}\frac{1}{2}|\vec v|^2\right)+\rho\vec E+\frac{1}{c}\vec j\times\vec B-\nabla_{\vec x}p=\\
=-\mu(\vec u-\vec v)\times curl_{\vec x}\vec v
+\mu\left(\partial_t\vec v+ d_{\vec x}\vec v\cdot\vec
v\right)+\rho\vec E+\frac{1}{c}\vec j\times\vec B-\nabla_{\vec x}p,
\end{multline}
and we rewrite
\er{MaxVacFull1ninshtrgravortghhghgjkgghklhjgkghghjjkjhjkkggjkhjkhjjhhfhjhkjkhbbgjhzzzyyykkknnnNWBWHWPPNjjhjhjhjh1}
as:
\begin{multline}\label{MaxVacFull1ninshtrgravortghhghgjkgghklhjgkghghjjkjhjkkggjkhjkhjjhhfhjhkjkhbbgjhzzzyyykkknnnNWBWHWPPNjjhjhjhjhugyyy}
\frac{\partial E}{\partial t}+\Div_{\vec x}\left\{E\vec
u\right\}=\mu\left(\frac{\partial}{\partial
t}\left(\frac{E}{\mu}\right)+\vec u\cdot\nabla_{\vec
x}\left(\frac{E}{\mu}\right)\right)\\ =-\left(\Div_{\vec x}\vec
u\right)p-\Div_{\vec x}\vec q+\left(\vec j-\rho\vec
u\right)\cdot\left(\vec E+\frac{1}{c}\vec u\times\vec B\right).
\end{multline}
Moreover, the internal energy per unit mass is a function of the
density $\mu$, pressure $p$ and Kelvin temperature $T$:
\begin{equation}\label{MaxVacFull1ninshtrgravortghhghgjkgghklhjgkghghjjkjhjkkggjkhjkhjjhhfhjhkjkhbbgjhzzzyyykkknnnNWBWHWPPNjjhjhjhjhplllhjggjlh}
\frac{E}{\mu}=U\left(\mu,p,T\right).
\end{equation}
There is also a state equation of the form:
\begin{equation}\label{MaxVacFull1ninsmljjlhtrgravortghhghgjkgghklhjgkghghjjkjhjkkggjkhjkhjjhhfhjhkjkhbbgjhzzzyyykkknnnNWBWHWPPNjjhjhjhjhplllhjggjlh}
T=g\left(\mu,p\right).
\end{equation}
We remind that in the case of the simplest ideal gas
\er{MaxVacFull1ninshtrgravortghhghgjkgghklhjgkghghjjkjhjkkggjkhjkhjjhhfhjhkjkhbbgjhzzzyyykkknnnNWBWHWPPNjjhjhjhjhplllhjggjlh}
takes particular form of \er{noninchredPPNbjjhjjkhkhkljjljljkl}:
\begin{equation}\label{noninchredPPNbjjhjjkhkhkljjljljklhjhjhjghhgghjhhj}
\frac{E}{\mu}\,=\,\frac{\gamma_0}{m_0}\,\,k\,T,
\end{equation}
and
\er{MaxVacFull1ninsmljjlhtrgravortghhghgjkgghklhjgkghghjjkjhjkkggjkhjkhjjhhfhjhkjkhbbgjhzzzyyykkknnnNWBWHWPPNjjhjhjhjhplllhjggjlh}
takes particular form of \er{noninchredPPNbjjhjjkhkhlkkkkl}:
\begin{equation}\label{noninchredPPNbjjhjjkhkhlkkkklhjgghgh}
T\,=\,\frac{m_0}{k}\,\frac{p}{\mu},
\end{equation}
where $m_0$ is the mass of the single molecule of the given gas and
$k$ is the Boltzmann constant and $\gamma_0>0$ is a constant that
depends on the kind of the gas (for the monatomic gas we have
$\gamma_0=\frac{3}{2}$). So, by inserting
\er{MaxVacFull1ninsmljjlhtrgravortghhghgjkgghklhjgkghghjjkjhjkkggjkhjkhjjhhfhjhkjkhbbgjhzzzyyykkknnnNWBWHWPPNjjhjhjhjhplllhjggjlh}
into
\er{MaxVacFull1ninshtrgravortghhghgjkgghklhjgkghghjjkjhjkkggjkhjkhjjhhfhjhkjkhbbgjhzzzyyykkknnnNWBWHWPPNjjhjhjhjhplllhjggjlh}
we have:
\begin{equation}\label{MaxVacFull1nkkkinshtrgravortghhghgjkgghklhjgkghghjjkjhjkkggjkhjkhjjhhfhjhkjkhbbgjhzzzyyykkknnhjhhnNWBWHWPPNjjhjhjhjhplllhjggjll;;lijiih}
\frac{E}{\mu}
=U\left(\mu,p,g\left(\mu,p\right)\right):=F\left(\mu,p\right).
\end{equation}
In particular, in the case where
\er{noninchredPPNbjjhjjkhkhkljjljljklhjhjhjghhgghjhhj} and
\er{noninchredPPNbjjhjjkhkhlkkkklhjgghgh} are valid we have
\begin{equation}\label{MaxVacFull1nkkkinshtrgravortghhghgjkgghklhjgkghghjjkjhjkkggjkhjkhjjhhfhjhkjkhbbgjhzzzyyykkknnhjhhnNWBWHWPPNjjhjhjhjhplllhjggjll;;lijiihkhhkgjgj}
\frac{E}{\mu}=F\left(\mu,p\right):=\gamma_0\,\frac{p}{\mu}.
\end{equation}

Next we assume that
\begin{equation}\label{MaxVacFull1ninshtrgravortghhghgjkgghklhjgkghghjjkjhjkkggjkhjkhjjhhfhjhkjkhbbgjhzzzyyykkknnnNWBWHWPPNjjhjhjhjhplllhjggjl}
\vec u=\vec u_0+\vec u_1,\quad\mu=\mu_0+\mu_1,\quad p=p_0+p_1.
\end{equation}
where $\vec u_0(\vec x,t),\mu_0(\vec x,t),p_0(\vec x,t)$ are the
averages of $\vec u(\vec x,t),\mu(\vec x,t),p(\vec x,t)$ on small
spatial and temporal intervals, surrounding the point $(\vec x,t)$.
Although we assume these intervals of space and time to be very
small, we also assume them to be quite \underline{macroscopic}. We
call $\vec u_1,\mu_1,p_1$ the oscillating parts of $\vec u,\mu,p$
and we assume that they are small with respect to the averages $\vec
u_0,\mu_0,p_0$ i.e. we have:
\begin{equation}\label{MaxVacFull1yiyiyninshtrgravortghhghgjkgghklhjgkghghjjkjhjkkggjkhjkhjjhhfhjhkjkhbbgjhzzzyyykkknnnNWBWHWPPNjjhjhjhjhplllhjggjl;k}
|\vec u_1|\ll|\vec u_0|,\quad|\mu_1|\ll |\mu_0|,\quad |p_1|\ll
|p_0|.
\end{equation}
However, we assume that $\vec u_1,\mu_1,p_1$ are highly oscillate
and thus they changes spatially and temporary much faster than the
averages  $\vec u_0,\mu_0,p_0$ and the fields $\vec v, \vec E,\vec
B,\rho,\vec j$ i.e. we have:
\begin{multline}\label{MaxVacFull1yiyiyninshtrgravortghhghgjkgghklhjghjhnkkghghjjkjhjkkggjkhjkhjjhhfhjhkjkhbbgjhzzzyyykkknnnNWBWHWPPNjjhjhjhjhplllhjggjl;k}
\frac{\left|d_{\vec x}\left(\rho\vec E+\frac{1}{c}\vec j\times\vec
B\right)\right|}{\left|\rho\vec E+\frac{1}{c}\vec j\times\vec
B\right|}+\frac{|d_{\vec x}\vec v|}{|\vec v|}+\frac{|d_{\vec
x}\mu_0|}{|\mu_0|}+\frac{|d_{\vec x}p_0|}{|p_0|}+\frac{|d_{\vec
x}\vec u_0|}{|\vec
u_0|}\ll \min{\left\{\frac{|d_{\vec x}\vec u_1|}{|\vec u_1|},\frac{|d_{\vec x}\mu_1|}{|\mu_1|},\frac{|d_{\vec x}p_1|}{|p_1|}\right\}},\\
\quad\quad
\frac{\left|\partial_t\left(\rho\vec E+\frac{1}{c}\vec j\times\vec
B\right)\right|}{\left|\rho\vec E+\frac{1}{c}\vec j\times\vec
B\right|}+\frac{|\partial_t\vec v|}{|\vec
v|}+\frac{|\partial_t\mu_0|}{|\mu_0|}+\frac{|\partial_t
p_0|}{|p_0|}+\frac{|\partial_t\vec u_0|}{|\vec u_0|}\ll
\min{\left\{\frac{|\partial_t\vec u_1|}{|\vec
u_1|},\frac{|\partial_t\mu_1|}{|\mu_1|},\frac{|\partial_t
p_1|}{|p_1|}\right\}},\\ \frac{|\vec u_1|}{|\vec u_0|}\ll
\min{\left\{\frac{|d_{\vec x}\vec u_1|}{|d_{\vec x}\vec
u_0|},\frac{|d_{\vec x}\vec u_1|}{|d_{\vec x}\vec v|}\right\}},
\quad\quad \frac{|\mu_1|}{|\mu_0|}+\frac{|p_1|}{|p_0|}\ll
\min{\left\{\frac{|d_{\vec x}p_1|}{|d_{\vec x}p_0|},\frac{|d_{\vec
x}\mu_1|}{|d_{\vec x}\mu_0|}\right\}}\\ \quad\quad\text{and}
\quad\quad \frac{|\mu_1|}{|\mu_0|}\ll\frac{|d_{\vec x}\vec u_1||\vec
u_0||\mu_0|}{\left|\rho\vec E+\frac{1}{c}\vec j\times\vec B\right|}.
\end{multline}
Finally, we assume that the fields $\vec v, \vec E,\vec B,\rho,\vec
j$ and $(\Div_{\vec x}\vec q)$ change slowly with respect to the
oscillations of  $\vec u_1,\mu_1,p_1$ and thus we assume that $\vec
v, \vec E,\vec B,\rho,\vec j$ and $(\Div_{\vec x}\vec q)$ can be
replaced by their spatial and temporal averages. Note that
$\mu,p,\mu_0,p_0,\mu_1,p_1$ behave like \underline{proper scalar}
fields and $\vec u,\vec u_0$ behave like \underline{speed-like
vector} fields under the change of cartesian coordinate systems.
Thus, since $\vec u_1=\vec u-\vec u_0$, we deduce that $\vec u_1$
behaves like a \underline{proper vector} field  under the change of
cartesian coordinate systems. Finally, note that obviously the
averages of $\vec u_1,\mu_1,p_1$ vanish.

Next we would like to approximate
\er{MaxVacFull1ninshtrgravortghhghgjkgghklhjgkghghjjkjhjkkggjkhjkhjjhhfhjhkjkhbbgjhzzzyyykkknnnNWBWHWPPNjjhuhh}.
In the rough level of approximation we could just replace $\vec
u,\mu,p$ by  $\vec u_0,\mu_0,p_0$. However, we would like to use a
more delicate approximation, taking into the account the first order
terms with $\vec u_1,\mu_1,p_1$. Then, by inserting
\er{MaxVacFull1ninshtrgravortghhghgjkgghklhjgkghghjjkjhjkkggjkhjkhjjhhfhjhkjkhbbgjhzzzyyykkknnnNWBWHWPPNjjhjhjhjhplllhjggjl}
into
\er{MaxVacFull1ninshtrgravortghhghgjkgghklhjgkghghjjkjhjkkggjkhjkhjjhhfhjhkjkhbbgjhzzzyyykkknnnNWBWHWPPNjjhuhh}
and using
\er{MaxVacFull1yiyiyninshtrgravortghhghgjkgghklhjgkghghjjkjhjkkggjkhjkhjjhhfhjhkjkhbbgjhzzzyyykkknnnNWBWHWPPNjjhjhjhjhplllhjggjl;k}
we conclude:
\begin{multline}\label{MaxVacFull1ninshtrgravortghhghgjkgghklhjgkghghjjkjhjkkggjkhjkhjjhhfhjhkjhjjhkhbbgjlmkmmhzzzyyykkknnnNWBWHWPPNjjlllkkkkkklbbbnjh}
\mu\left(\frac{\partial\vec u}{\partial t}+d_{\vec x}\vec u\cdot
\vec u\right)\approx\\ \mu_0\left(\frac{\partial\vec u_1}{\partial
t}+d_{\vec x}\vec u_0\cdot \vec u_1+d_{\vec x}\vec u_1\cdot \vec
u_0\right)+\mu_1\left(\frac{\partial\vec u_0}{\partial t}+d_{\vec
x}\vec u_0\cdot \vec u_0\right)+\mu_0\left(\frac{\partial\vec
u_0}{\partial t}+d_{\vec x}\vec u_0\cdot \vec u_0\right)
\approx\\-\mu_0\vec u_0\times curl_{\vec x}\vec
v+\mu_0\left(\partial_{t}\vec v+\nabla_{\vec x}\frac{1}{2}|\vec
v|^2\right)+\rho\vec E+\frac{1}{c}\vec j\times\vec B-\nabla_{\vec x}
p_0\\ +\mu_1\left(\partial_{t}\vec v+\nabla_{\vec x}\frac{1}{2}|\vec
v|^2\right) -\mu_1\vec u_0\times curl_{\vec x}\vec v-\mu_0\vec
u_1\times curl_{\vec x}\vec v-\nabla_{\vec x} p_1.
\end{multline}
On the other hand by inserting
\er{MaxVacFull1ninshtrgravortghhghgjkgghklhjgkghghjjkjhjkkggjkhjkhjjhhfhjhkjkhbbgjhzzzyyykkknnnNWBWHWPPNjjhjhjhjhplllhjggjl}
into
\er{MaxVacFull1ninshtrgravortghhghgjkgghklhjgkghghjjkjhjkkggjkhjkhjjhhfhjhkhjjkhjgzzzyyykkknnnNWBWHWPPNjj1}
and using
\er{MaxVacFull1yiyiyninshtrgravortghhghgjkgghklhjgkghghjjkjhjkkggjkhjkhjjhhfhjhkjkhbbgjhzzzyyykkknnnNWBWHWPPNjjhjhjhjhplllhjggjl;k}
we have
\begin{multline}\label{MaxVacFull1ninshtrgravortghhghgjkgghklhjgkghghjjkjhjkkggjkhjkhjjhhfhjhkhjjkhjgzzzyyykkknnnNWBWHWPPNjjlkjljjjkhjhjhjhjhjbh1}
0=\left(\frac{\partial\mu}{\partial t}+\vec
u\cdot\nabla_{\vec x}\mu\right)+\mu\,div_{\vec x}\vec u\approx\\
\left(\frac{\partial\mu_0}{\partial t}+\vec u_0\cdot\nabla_{\vec
x}\mu_0\right)+\left(\frac{\partial\mu_1}{\partial t}+\vec
u_0\cdot\nabla_{\vec x}\mu_1+\vec u_1\cdot\nabla_{\vec
x}\mu_0\right)+\mu_0\,div_{\vec x}\vec u_0+\mu_0\,div_{\vec x}\vec
u_1+\mu_1\,div_{\vec x}\vec u_0\approx 0.
\end{multline}
In particular, averaging of
\er{MaxVacFull1ninshtrgravortghhghgjkgghklhjgkghghjjkjhjkkggjkhjkhjjhhfhjhkjhjjhkhbbgjlmkmmhzzzyyykkknnnNWBWHWPPNjjlllkkkkkklbbbnjh}
gives:
\begin{multline}\label{MaxVacFull1ninshtrgravortghhghgjkgghklhjgkghghjjkjhjkkggjkhjkhjjhhfhjhkjhjjhkhbbgjlmkmmhzzzyyykkknnnNWBWHWPPNjjlllkkkkkklhjhjh}
\mu_0\left(\frac{\partial\vec u_0}{\partial t}+d_{\vec x}\vec
u_0\cdot \vec u_0\right) \approx-\mu_0\vec u_0\times curl_{\vec
x}\vec v+\mu_0\left(\partial_{t}\vec v+\nabla_{\vec
x}\frac{1}{2}|\vec v|^2\right)+\rho\vec E+\frac{1}{c}\vec
j\times\vec B-\nabla_{\vec x} p_0,
\end{multline}
and averaging of
\er{MaxVacFull1ninshtrgravortghhghgjkgghklhjgkghghjjkjhjkkggjkhjkhjjhhfhjhkhjjkhjgzzzyyykkknnnNWBWHWPPNjjlkjljjjkhjhjhjhjhjbh1}
gives
\begin{equation}\label{MaxVacFull1ninshtrgravortghhghgjkgghklhjgkghghjjkjhjkkggjkhjkhjjhhfhjhkhjjkhjgzzzyyykkknnnNWBWHWPPNjjlkjljjjkhjhjhjhjhjbh}
\left(\frac{\partial\mu_0}{\partial t}+\vec u_0\cdot\nabla_{\vec
x}\mu_0\right)+\mu_0\,div_{\vec x}\vec u_0\approx 0.
\end{equation}
Here we used the fact that that $\vec v, \vec E,\vec B,\rho,\vec j$
can be replaced by their spatial and temporal averages. Thus,
subtracting
\er{MaxVacFull1ninshtrgravortghhghgjkgghklhjgkghghjjkjhjkkggjkhjkhjjhhfhjhkjhjjhkhbbgjlmkmmhzzzyyykkknnnNWBWHWPPNjjlllkkkkkklhjhjh}
from the unaveraged
\er{MaxVacFull1ninshtrgravortghhghgjkgghklhjgkghghjjkjhjkkggjkhjkhjjhhfhjhkjhjjhkhbbgjlmkmmhzzzyyykkknnnNWBWHWPPNjjlllkkkkkklbbbnjh}
gives:
\begin{multline}\label{MaxVacFull1ninshtrgravortghhghgjkgghklhjgkghghjjkjhjkkggjkhjkhjjhhfhjhkjhjjhkhbbgjlmkmmhzzzyyykkknnnNWBWHWPPNjjlllkkkkkklbbbnjhkhkhhjjoujuijkjk}
\mu_0\left(\frac{\partial\vec u_1}{\partial t}+d_{\vec x}\vec
u_0\cdot \vec u_1+d_{\vec x}\vec u_1\cdot \vec
u_0\right)+\mu_1\left(\frac{\partial\vec u_0}{\partial t}+d_{\vec
x}\vec u_0\cdot \vec u_0\right) \approx\\
+\mu_1\left(\partial_{t}\vec v+\nabla_{\vec x}\frac{1}{2}|\vec
v|^2\right) -\mu_1\vec u_0\times curl_{\vec x}\vec v-\mu_0\vec
u_1\times curl_{\vec x}\vec v-\nabla_{\vec x} p_1,
\end{multline}
and thus, by
\er{MaxVacFull1ninshtrgravortghhghgjkgghklhjgkghghjjkjhjkkggjkhjkhjjhhfhjhkjhjjhkhbbgjlmkmmhzzzyyykkknnnNWBWHWPPNjjlllkkkkkklbbbnjhkhkhhjjoujuijkjk}
and
\er{MaxVacFull1ninshtrgravortghhghgjkgghklhjgkghghjjkjhjkkggjkhjkhjjhhfhjhkjhjjhkhbbgjlmkmmhzzzyyykkknnnNWBWHWPPNjjlllkkkkkklhjhjh}
we have:
\begin{multline}\label{MaxVacFull1ninshtrgravortghhghgjkgghklhjgkghghjjkjhjkkggjkhjkhjjhhfhjhkjhjjhkhbbgjlmkmmhzzzyyykkknnnNWBWHWPPNjjlllkkkkkklbbbnjhkhkhhj}
\mu_0\left(\frac{\partial\vec u_1}{\partial t}+d_{\vec x}\vec
u_0\cdot \vec u_1+d_{\vec x}\vec u_1\cdot \vec u_0\right) \approx\\
-\mu_0\vec u_1\times curl_{\vec x}\vec v-\nabla_{\vec x}
p_1+\frac{\mu_1}{\mu_0}\nabla_{\vec x} p_0
-\frac{\mu_1}{\mu_0}\left(\rho\vec E+\frac{1}{c}\vec j\times\vec
B\right).
\end{multline}
Furthermore, subtracting
\er{MaxVacFull1ninshtrgravortghhghgjkgghklhjgkghghjjkjhjkkggjkhjkhjjhhfhjhkhjjkhjgzzzyyykkknnnNWBWHWPPNjjlkjljjjkhjhjhjhjhjbh}
from the unaveraged
\er{MaxVacFull1ninshtrgravortghhghgjkgghklhjgkghghjjkjhjkkggjkhjkhjjhhfhjhkhjjkhjgzzzyyykkknnnNWBWHWPPNjjlkjljjjkhjhjhjhjhjbh1}
gives:
\begin{equation}\label{MaxVacFull1ninshtrgravortghhghgjkgghklhjgkghghjjkjhjkkggjkhjkhjjhhfhjhkhjjkhjgzzzyyykkknnnNWBWHWPPNjjlkjljjjkhjhjhjhjhjbhjihhhhhjh}
\left(\frac{\partial\mu_1}{\partial t}+\vec u_0\cdot\nabla_{\vec
x}\mu_1+\vec u_1\cdot\nabla_{\vec x}\mu_0\right)+\mu_0\,div_{\vec
x}\vec u_1+\mu_1\,div_{\vec x}\vec u_0\approx 0.
\end{equation}
Therefore, using
\er{MaxVacFull1yiyiyninshtrgravortghhghgjkgghklhjghjhnkkghghjjkjhjkkggjkhjkhjjhhfhjhkjkhbbgjhzzzyyykkknnnNWBWHWPPNjjhjhjhjhplllhjggjl;k}
we further approximate
\er{MaxVacFull1ninshtrgravortghhghgjkgghklhjgkghghjjkjhjkkggjkhjkhjjhhfhjhkjhjjhkhbbgjlmkmmhzzzyyykkknnnNWBWHWPPNjjlllkkkkkklbbbnjhkhkhhj}
as:
\begin{equation}\label{MaxVacFull1ninshtrgravortghhghgjkgghklhjgkghghjjkjhjkkggjkhjkhjjhhfhjhkjhjjhkhbbgjlmkmmhzzzyyykkknnnNWBWHWPPNjjlllkkkkkkljkkjjkjk11}
\mu_0\left(\frac{\partial\vec u_1}{\partial t}+d_{\vec x}\vec
u_1\cdot \vec u_0\right) \approx-\nabla_{\vec x} p_1.
\end{equation}
Thus, again using
\er{MaxVacFull1yiyiyninshtrgravortghhghgjkgghklhjghjhnkkghghjjkjhjkkggjkhjkhjjhhfhjhkjkhbbgjhzzzyyykkknnnNWBWHWPPNjjhjhjhjhplllhjggjl;k},
we approximate
\er{MaxVacFull1ninshtrgravortghhghgjkgghklhjgkghghjjkjhjkkggjkhjkhjjhhfhjhkjhjjhkhbbgjlmkmmhzzzyyykkknnnNWBWHWPPNjjlllkkkkkkljkkjjkjk11}
as:
\begin{equation}\label{MaxVacFull1ninshtrgravortghhghgjkgghklhjgkghghjjkjhjkkggjkhjkhjjhhfhjhkjhjjhkhbbgjlmkmmhzzzyyykkknnnNWBWHWPPNjjlllkkkkkkljkkjjkjkuhg}
\mu_0\left(\frac{\partial\vec u_1}{\partial t}+(\Div_{\vec x}\vec
u_0)\,\vec u_1+d_{\vec x}\vec u_1\cdot \vec u_0-d_{\vec x}\vec
u_0\cdot \vec u_1\right) \approx-\nabla_{\vec x} p_1.
\end{equation}
Then, by \er{apfrm6} we rewrite
\er{MaxVacFull1ninshtrgravortghhghgjkgghklhjgkghghjjkjhjkkggjkhjkhjjhhfhjhkjhjjhkhbbgjlmkmmhzzzyyykkknnnNWBWHWPPNjjlllkkkkkkljkkjjkjkuhg}
as:
\begin{equation}\label{MaxVacFull1ninshtrgravortghhghgjkgghklhjgkghghjjkjhjkkggjkhjkhjjhhfhjhkjhjjhkhbbgjlmkmmhzzzyyykkknnnNWBWHWPPNjjlllkkkkkkljkkjjkjk}
\mu_0\left(\frac{\partial\vec u_1}{\partial t}-curl_{\vec
x}\left\{\vec u_0\times\vec u_1\right\}+(\Div_{\vec x}\vec
u_1)\,\vec u_0\right) \approx-\nabla_{\vec x} p_1.
\end{equation}
Moreover, using
\er{MaxVacFull1yiyiyninshtrgravortghhghgjkgghklhjghjhnkkghghjjkjhjkkggjkhjkhjjhhfhjhkjkhbbgjhzzzyyykkknnnNWBWHWPPNjjhjhjhjhplllhjggjl;k}
we further approximate
\er{MaxVacFull1ninshtrgravortghhghgjkgghklhjgkghghjjkjhjkkggjkhjkhjjhhfhjhkhjjkhjgzzzyyykkknnnNWBWHWPPNjjlkjljjjkhjhjhjhjhjbhjihhhhhjh}
as:
\begin{equation}\label{MaxVacFull1ninshtrgravortghhghgjkgghklhjgkghghjjkjhjkkggjkhjkhjjhhfhjhkhjjkhjgzzzyyykkknnnNWBWHWPPNjjlkjljjjkijijji1}
\left(\frac{\partial\mu_1}{\partial t}+\vec u_0\cdot\nabla_{\vec
x}\mu_1\right)+\mu_0\,div_{\vec x}\vec u_1\approx 0,
\end{equation}
and thus,
\begin{equation}\label{MaxVacFull1ninshtrgravortghhghgjkgghklhjgkghghjjkjhjkkggjkhjkhjjhhfhjhkhjjkhjgzzzyyykkknnnNWBWHWPPNjjlkjljjjkijijji}
\frac{1}{\mu_0}\left(\frac{\partial\mu_1}{\partial t}+\vec
u_0\cdot\nabla_{\vec x}\mu_1\right)+div_{\vec x}\vec u_1\approx 0.
\end{equation}
In particular, taking the divergence of both sides of
\er{MaxVacFull1ninshtrgravortghhghgjkgghklhjgkghghjjkjhjkkggjkhjkhjjhhfhjhkjhjjhkhbbgjlmkmmhzzzyyykkknnnNWBWHWPPNjjlllkkkkkkljkkjjkjk}
and using again
\er{MaxVacFull1yiyiyninshtrgravortghhghgjkgghklhjghjhnkkghghjjkjhjkkggjkhjkhjjhhfhjhkjkhbbgjhzzzyyykkknnnNWBWHWPPNjjhjhjhjhplllhjggjl;k},
gives
\begin{equation}\label{MaxVacFull1ninshtrgravortghhghgjkgghklhjgkghghjjkjhjkkggjkhjkhjjhhfhjhkjhjjhkhbbgjlmkmmhzzzyyykkknnnNWBWHWPPNjjlllkkkkkkljkkjjkjkhjjhhggh}
\mu_0\left(\frac{\partial}{\partial t}\left(\Div_{\vec x}\vec
u_1\right)+\Div_{\vec x}\left\{\left(\Div_{\vec x}\vec u_1\right)
\vec u_0\right\}\right) \approx-\Delta_{\vec x} p_1.
\end{equation}
Therefore, inserting
\er{MaxVacFull1ninshtrgravortghhghgjkgghklhjgkghghjjkjhjkkggjkhjkhjjhhfhjhkhjjkhjgzzzyyykkknnnNWBWHWPPNjjlkjljjjkijijji}
into
\er{MaxVacFull1ninshtrgravortghhghgjkgghklhjgkghghjjkjhjkkggjkhjkhjjhhfhjhkjhjjhkhbbgjlmkmmhzzzyyykkknnnNWBWHWPPNjjlllkkkkkkljkkjjkjkhjjhhggh}
and using again
\er{MaxVacFull1yiyiyninshtrgravortghhghgjkgghklhjghjhnkkghghjjkjhjkkggjkhjkhjjhhfhjhkjkhbbgjhzzzyyykkknnnNWBWHWPPNjjhjhjhjhplllhjggjl;k},
we deduce
\begin{equation}\label{MaxVacFull1ninshtrgravortghhghgjkgghklhjgkghghjjkjhjkkggjkhjkhjjhhfhjhkjhjjhkhbbgjlmkmmhzzzyyykkknnnNWBWHWPPNjjlllkkkkkkljkkjjkjkhjjhhgghgggg}
\frac{\partial}{\partial t}\left(\frac{\partial\mu_1}{\partial
t}+\vec u_0\cdot\nabla_{\vec x}\mu_1\right)+\Div_{\vec
x}\left\{\left(\frac{\partial\mu_1}{\partial t}+\vec
u_0\cdot\nabla_{\vec x}\mu_1\right) \vec u_0\right\}
\approx\Delta_{\vec x} p_1.
\end{equation}
Note that, as before, it can be easily proved that
\er{MaxVacFull1ninshtrgravortghhghgjkgghklhjgkghghjjkjhjkkggjkhjkhjjhhfhjhkjhjjhkhbbgjlmkmmhzzzyyykkknnnNWBWHWPPNjjlllkkkkkkljkkjjkjkhjjhhgghgggg}
and
\er{MaxVacFull1ninshtrgravortghhghgjkgghklhjgkghghjjkjhjkkggjkhjkhjjhhfhjhkjhjjhkhbbgjlmkmmhzzzyyykkknnnNWBWHWPPNjjlllkkkkkkljkkjjkjk}
are invariant under the change of inertial or non-inertial cartesian
coordinate system. Thus,
\er{MaxVacFull1ninshtrgravortghhghgjkgghklhjgkghghjjkjhjkkggjkhjkhjjhhfhjhkjhjjhkhbbgjlmkmmhzzzyyykkknnnNWBWHWPPNjjlllkkkkkkljkkjjkjkhjjhhgghgggg}
and
\er{MaxVacFull1ninshtrgravortghhghgjkgghklhjgkghghjjkjhjkkggjkhjkhjjhhfhjhkjhjjhkhbbgjlmkmmhzzzyyykkknnnNWBWHWPPNjjlllkkkkkkljkkjjkjk}
are still valid if
\er{MaxVacFull1yiyiyninshtrgravortghhghgjkgghklhjgkghghjjkjhjkkggjkhjkhjjhhfhjhkjkhbbgjhzzzyyykkknnnNWBWHWPPNjjhjhjhjhplllhjggjl;k}
and
\er{MaxVacFull1yiyiyninshtrgravortghhghgjkgghklhjghjhnkkghghjjkjhjkkggjkhjkhjjhhfhjhkjkhbbgjhzzzyyykkknnnNWBWHWPPNjjhjhjhjhplllhjggjl;k}
hold only in some specific (possibly artificial) inertial or
non-inertial cartesian coordinate system.

Next we proceed in two cases:
\begin{itemize}
\item[{\bf (i)}]
In the case of the simple models of inviscid gas.
\item[{\bf (ii)}]
In the case of inviscid barotropic fluid.
\end{itemize}
In the case of inviscid gas, approximating
\er{MaxVacFull1ninshtrgravortghhghgjkgghklhjgkghghjjkjhjkkggjkhjkhjjhhfhjhkjkhbbgjhzzzyyykkknnnNWBWHWPPNjjhjhjhjhugyyy},
by taking into account up to the first order terms with $\vec
u_1,\mu_1,p_1$ and using
\er{MaxVacFull1nkkkinshtrgravortghhghgjkgghklhjgkghghjjkjhjkkggjkhjkhjjhhfhjhkjkhbbgjhzzzyyykkknnhjhhnNWBWHWPPNjjhjhjhjhplllhjggjll;;lijiih},
we infer:
\begin{multline}\label{MaxVacFull1ninshtrgravortghhghgjkgghklhjnnmnmmngkghghjjkjhjkkggjkhjkhjjhhfhjhkjkhbbgjhzzzyyykkknnnNWBWHWPPNjjhjhjhjhghg}
\mu\left(\frac{\partial}{\partial t}\left(\frac{E}{\mu}\right)+\vec
u\cdot\nabla_{\vec x}\left(\frac{E}{\mu}\right)\right)=\\
(\mu_0+\mu_1)\left(\frac{\partial}{\partial
t}\left(F\left(\mu_0+\mu_1,p_0+p_1\right)\right)+(\vec u_0+\vec
u_1)\cdot\nabla_{\vec
x}\left(F\left(\mu_0+\mu_1,p_0+p_1\right)\right)\right)\\
\approx \mu_1\left(\frac{\partial}{\partial
t}\left(F\left(\mu_0,p_0\right)\right)+\vec u_0\cdot\nabla_{\vec
x}\left(F\left(\mu_0,p_0\right)\right)\right)\\
+\mu_0\left(\frac{\partial}{\partial
t}\left(F\left(\mu_0+\mu_1,p_0+p_1\right)\right)+\vec
u_0\cdot\nabla_{\vec
x}\left(F\left(\mu_0+\mu_1,p_0+p_1\right)\right)+\vec
u_1\cdot\nabla_{\vec x}\left(F\left(\mu_0,p_0\right)\right)\right)
\\
\approx \mu_0\left(\frac{\partial}{\partial
t}\left(F\left(\mu_0,p_0\right)\right)+\vec u_0\cdot\nabla_{\vec
x}\left(F\left(\mu_0,p_0\right)\right)\right)\\
+\mu_1\left(\frac{\partial}{\partial
t}\left(F\left(\mu_0,p_0\right)\right)+\vec u_0\cdot\nabla_{\vec
x}\left(F\left(\mu_0,p_0\right)\right)\right)+\mu_0\vec
u_1\cdot\nabla_{\vec x}\left(F\left(\mu_0,p_0\right)\right)\\
+\mu_0\left(\frac{\partial}{\partial t}\left(\mu_1\frac{\partial
F}{\partial\mu}\left(\mu_0,p_0\right)+p_1\frac{\partial F}{\partial
p}\left(\mu_0,p_0\right)\right)+\vec u_0\cdot\nabla_{\vec
x}\left(\mu_1\frac{\partial
F}{\partial\mu}\left(\mu_0,p_0\right)+p_1\frac{\partial F}{\partial
p}\left(\mu_0,p_0\right)\right)\right)
\\
\approx -\left(div_{\vec x}\vec u_1\right)p_0-\left(div_{\vec x}\vec
u_0\right)p_1-\left(div_{\vec x}\vec u_0\right)p_0-\Div_{\vec x}\vec
q+\left(\vec j-\rho\vec u\right)\cdot\left(\vec E+\frac{1}{c}\vec
u\times\vec B\right).
\end{multline}
In particular, averaging of
\er{MaxVacFull1ninshtrgravortghhghgjkgghklhjnnmnmmngkghghjjkjhjkkggjkhjkhjjhhfhjhkjkhbbgjhzzzyyykkknnnNWBWHWPPNjjhjhjhjhghg}
gives:
\begin{multline}\label{MaxVacFull1ninshtrgravortghhghgjkgghklhjnnmnmmngkghghjjkjhjkkggjkhjkhjjhhfhjhkjkhbbgjhzzzyyykkknnnNWBWHWPPNjjhjhjhjhghghhj}
\mu_0\left(\frac{\partial}{\partial
t}\left(F\left(\mu_0,p_0\right)\right)+\vec u_0\cdot\nabla_{\vec
x}\left(F\left(\mu_0,p_0\right)\right)\right) \approx
-\left(div_{\vec x}\vec u_0\right)p_0-\Div_{\vec x}\vec q+\left(\vec
j-\rho\vec u\right)\cdot\left(\vec E+\frac{1}{c}\vec u\times\vec
B\right).
\end{multline}
Here we again used the fact that that $\vec v, \vec E,\vec
B,\rho,\vec j$ and $\Div_{\vec x}\vec q$ can be replaced by their
spatial and temporal averages. Thus, subtracting
\er{MaxVacFull1ninshtrgravortghhghgjkgghklhjnnmnmmngkghghjjkjhjkkggjkhjkhjjhhfhjhkjkhbbgjhzzzyyykkknnnNWBWHWPPNjjhjhjhjhghghhj}
from the unaveraged
\er{MaxVacFull1ninshtrgravortghhghgjkgghklhjnnmnmmngkghghjjkjhjkkggjkhjkhjjhhfhjhkjkhbbgjhzzzyyykkknnnNWBWHWPPNjjhjhjhjhghg}
gives:
\begin{multline}\label{MaxVacFull1ninshtrgravortghhghgjkgghklhjnnmnmmngkghghjjkjhjkkggjkhjkhjjhhfhjhkjkhbbgjhzzzyyykkknnnNWBWHWPPNjjhjhjhjhghgjkhkhlkkhhk}
-\left(div_{\vec x}\vec u_1\right)p_0-\left(div_{\vec x}\vec
u_0\right)p_1\approx\mu_1\left(\frac{\partial}{\partial
t}\left(F\left(\mu_0,p_0\right)\right)+\vec u_0\cdot\nabla_{\vec
x}\left(F\left(\mu_0,p_0\right)\right)\right)+\mu_0\vec
u_1\cdot\nabla_{\vec x}\left(F\left(\mu_0,p_0\right)\right)\\
+\mu_0\left(\frac{\partial}{\partial t}\left(\mu_1\frac{\partial
F}{\partial\mu}\left(\mu_0,p_0\right)+p_1\frac{\partial F}{\partial
p}\left(\mu_0,p_0\right)\right)+\vec u_0\cdot\nabla_{\vec
x}\left(\mu_1\frac{\partial
F}{\partial\mu}\left(\mu_0,p_0\right)+p_1\frac{\partial F}{\partial
p}\left(\mu_0,p_0\right)\right)\right).
\end{multline}
Therefore, using
\er{MaxVacFull1yiyiyninshtrgravortghhghgjkgghklhjghjhnkkghghjjkjhjkkggjkhjkhjjhhfhjhkjkhbbgjhzzzyyykkknnnNWBWHWPPNjjhjhjhjhplllhjggjl;k}
we further approximate
\er{MaxVacFull1ninshtrgravortghhghgjkgghklhjnnmnmmngkghghjjkjhjkkggjkhjkhjjhhfhjhkjkhbbgjhzzzyyykkknnnNWBWHWPPNjjhjhjhjhghgjkhkhlkkhhk}
as:
\begin{multline}\label{MaxVacFull1ninshtrgravortghhghgjkgghklhjnnmnmmngkghghjjkjhjkkggjkhjkhjjhhfhjhkjkhbbgjhzzzyyykkknnnNWBWHWPPNjjhjhjhjhghgjkhkhlkkhhkhkhh}
\mu_0\frac{\partial
F}{\partial\mu}\left(\mu_0,p_0\right)\left(\frac{\partial\mu_1}{\partial
t}+\vec u_0\cdot\nabla_{\vec x}\mu_1\right)+\mu_0\frac{\partial
F}{\partial p}\left(\mu_0,p_0\right)\left(\frac{\partial p}{\partial
t}+\vec u_0\cdot\nabla_{\vec x}p_1\right) \approx -\left(div_{\vec
x}\vec u_1\right)p_0.
\end{multline}
Thus, inserting
\er{MaxVacFull1ninshtrgravortghhghgjkgghklhjgkghghjjkjhjkkggjkhjkhjjhhfhjhkhjjkhjgzzzyyykkknnnNWBWHWPPNjjlkjljjjkijijji}
into
\er{MaxVacFull1ninshtrgravortghhghgjkgghklhjnnmnmmngkghghjjkjhjkkggjkhjkhjjhhfhjhkjkhbbgjhzzzyyykkknnnNWBWHWPPNjjhjhjhjhghgjkhkhlkkhhkhkhh}
gives:
\begin{multline}\label{MaxVacFull1ninshtrgravortghhghgjkgghklhjnnmnmmngkghghjjkjhjkkggjkhjkhjjhhfhjhkjkhbbgjhzzzyyykkknnnNWBWHWPPNjjhjhjhjhghgjkhkhlkkhhkhkhhhkhhjgjgg}
\mu_0\frac{\partial
F}{\partial\mu}\left(\mu_0,p_0\right)\left(\frac{\partial\mu_1}{\partial
t}+\vec u_0\cdot\nabla_{\vec x}\mu_1\right)+\mu_0\frac{\partial
F}{\partial p}\left(\mu_0,p_0\right)\left(\frac{\partial p}{\partial
t}+\vec u_0\cdot\nabla_{\vec x}p_1\right) \approx
\frac{p_0}{\mu_0}\left(\frac{\partial\mu_1}{\partial t}+\vec
u_0\cdot\nabla_{\vec x}\mu_1\right).
\end{multline}
So,
\begin{multline}\label{MaxVacFull1ninshtrgravortghhghgjkgghklhjnnmnmmngkghghjjkjhjkkggjkhjkhjjhhfhjhkjkhbbgjhzzzyyykkknnnNWBWHWPPNjjhjhjhjhghgjkhkhlkkhhkhkhhhkhhjgjggjkkhhk1}
\left(\mu_0\frac{\partial
F}{\partial\mu}\left(\mu_0,p_0\right)-\frac{p_0}{\mu_0}\right)\left(\frac{\partial\mu_1}{\partial
t}+\vec u_0\cdot\nabla_{\vec x}\mu_1\right)+\mu_0\frac{\partial
F}{\partial p}\left(\mu_0,p_0\right)\left(\frac{\partial p}{\partial
t}+\vec u_0\cdot\nabla_{\vec x}p_1\right) \approx 0.
\end{multline}
Therefore, denoting the scalar field  $g(\vec x,t)$ defined by
\begin{equation}\label{MaxVacFull1ninshtrgravortghhghgjkgghklhjnnmnmmngkghghjjkjhjkkggjkhjkhjjhhfhjhkjkhbbgjhzzzyyykkknnnNWBWHWPPNjjhjhjhjhghgjkhkhlkkhhkhkhhhkhhjgjggjkkhhk}
g:=\mu_0\frac{\partial F}{\partial
p}\left(\mu_0,p_0\right)p_1-\left(\frac{p_0}{\mu_0}-\mu_0\frac{\partial
F}{\partial\mu}\left(\mu_0,p_0\right)\right)\mu_1,
\end{equation}
inserting it into
\er{MaxVacFull1ninshtrgravortghhghgjkgghklhjnnmnmmngkghghjjkjhjkkggjkhjkhjjhhfhjhkjkhbbgjhzzzyyykkknnnNWBWHWPPNjjhjhjhjhghgjkhkhlkkhhkhkhhhkhhjgjggjkkhhk1}
and using
\er{MaxVacFull1yiyiyninshtrgravortghhghgjkgghklhjghjhnkkghghjjkjhjkkggjkhjkhjjhhfhjhkjkhbbgjhzzzyyykkknnnNWBWHWPPNjjhjhjhjhplllhjggjl;k}
again, we deduce
\begin{equation}\label{MaxVacFull1ninshtrgravortghhghgjkgghklhjnnmnmmngkghghjjkjhjkkggjkhjkhjjhhfhjhkjkhbbgjhzzzyyykkknnnNWBWHWPPNjjhjhjhjhghgjkhkhlkkhhkhkhhhkhhjgjggjkkhhkkhhkhkhg}
\frac{\partial g}{\partial t}+\vec u_0\cdot\nabla_{\vec x}g\approx
0.
\end{equation}
In particular we deduce
\begin{equation}\label{MaxVacFull1ninshtrgravortghhghgjkgghklhjnnmnmmngkghghjjkjhjkkggjkhjkhjjhhfhjhkjkhbbgjhzzzyyykkknnnNWBWHWPPNjjhjhjhjhghgjkhkhlkkhhkhkhhhkhhjgjggjkkhhkkhhkhkhgjgj}
\text{If}\quad g(\vec x,0)\approx 0\quad\text{then}\quad g(\vec
x,t)\approx 0\quad\forall\, t.
\end{equation}
In this case, by
\er{MaxVacFull1ninshtrgravortghhghgjkgghklhjnnmnmmngkghghjjkjhjkkggjkhjkhjjhhfhjhkjkhbbgjhzzzyyykkknnnNWBWHWPPNjjhjhjhjhghgjkhkhlkkhhkhkhhhkhhjgjggjkkhhk}
we deduce
\begin{equation}\label{MaxVacFull1ninshtrgravortghhghgjkgghklhjnnmnmmngkghghjjkjhjkkggjkhjkhjjhhfhjhkjkhbbgjhzzzyyykkknnnNWBWHWPPNjjhjhjhjhghgjkhkhlkkhhkhkhhhkhhjgjggjkkhhkjhjjhjhjhjhjjgg}
\mu_1\approx\left(\frac{p_0}{\mu_0}-\mu_0\frac{\partial
F}{\partial\mu}\left(\mu_0,p_0\right)\right)^{-1}\mu_0\frac{\partial
F}{\partial p}\left(\mu_0,p_0\right)\,p_1.
\end{equation}
Therefore, inserting
\er{MaxVacFull1ninshtrgravortghhghgjkgghklhjnnmnmmngkghghjjkjhjkkggjkhjkhjjhhfhjhkjkhbbgjhzzzyyykkknnnNWBWHWPPNjjhjhjhjhghgjkhkhlkkhhkhkhhhkhhjgjggjkkhhkjhjjhjhjhjhjjgg}
into
\er{MaxVacFull1ninshtrgravortghhghgjkgghklhjgkghghjjkjhjkkggjkhjkhjjhhfhjhkjhjjhkhbbgjlmkmmhzzzyyykkknnnNWBWHWPPNjjlllkkkkkkljkkjjkjkhjjhhgghgggg}
and using again
\er{MaxVacFull1yiyiyninshtrgravortghhghgjkgghklhjghjhnkkghghjjkjhjkkggjkhjkhjjhhfhjhkjkhbbgjhzzzyyykkknnnNWBWHWPPNjjhjhjhjhplllhjggjl;k}
we finally obtain the wave equation for the oscillating part of the
pressure $p_1$ of the form:
\begin{equation}\label{MaxVacFull1ninshtrgravortghhghgjkgghklhjgkghghjjkjhjkkggjkhjkhjjhhfhjhkjhjjhkhbbgjlmkmmhzzzyyykkknnnNWBWHWPPNjjlllkkkkkkljkkjjkjkhjjhhgghgggghjjhgjggh}
\frac{\partial}{\partial
t}\left\{\frac{1}{c^2_0}\left(\frac{\partial p_1}{\partial t}+\vec
u_0\cdot\nabla_{\vec x}p_1\right)\right\}+\Div_{\vec
x}\left\{\frac{1}{c^2_0}\left(\frac{\partial p_1}{\partial t}+\vec
u_0\cdot\nabla_{\vec x}p_1\right) \vec u_0\right\}
\approx\Delta_{\vec x} p_1,
\end{equation}
where we denote
\begin{equation}\label{MaxVacFull1ninshtrgravortghhghgjkgghklhjgkghghjjkjhjkkggjkhjkhjjhhfhjhkjhjjhkhbbgjlmkmmhzzzyyykkknnnNWBWHWPPNjjlllkkkkkkljkkjjkjkhjjhhgghgggghjjhgjgghhkhjhjh}
c_0:=\sqrt{\frac{\left(\frac{p_0}{\mu_0}-\mu_0\frac{\partial
F}{\partial\mu}\left(\mu_0,p_0\right)\right)}{\mu_0\frac{\partial
F}{\partial p}\left(\mu_0,p_0\right)}}.
\end{equation}
Moreover, the oscillating parts of the density $\mu_1$ and the
velocity $\vec u_1$ can be found from
\er{MaxVacFull1ninshtrgravortghhghgjkgghklhjnnmnmmngkghghjjkjhjkkggjkhjkhjjhhfhjhkjkhbbgjhzzzyyykkknnnNWBWHWPPNjjhjhjhjhghgjkhkhlkkhhkhkhhhkhhjgjggjkkhhkjhjjhjhjhjhjjgg}
and
\er{MaxVacFull1ninshtrgravortghhghgjkgghklhjgkghghjjkjhjkkggjkhjkhjjhhfhjhkjhjjhkhbbgjlmkmmhzzzyyykkknnnNWBWHWPPNjjlllkkkkkkljkkjjkjk}
respectively, i.e. we have
\begin{equation}\label{MaxVacFull1ninshtrgravortghhghgjkgghklhjnnmnmmngkghghjjkjhjkkggjkhjkhjjhhfhjhkjkhbbgjhzzzyyykkknnnNWBWHWPPNjjhjhjhjhghgjkhkhlkkhhkhkhhhkhhjgjggjkkhhkjhjjhjhjhjhjjggijjhj}
\mu_1\approx\left(\frac{p_0}{\mu_0}-\mu_0\frac{\partial
F}{\partial\mu}\left(\mu_0,p_0\right)\right)^{-1}\mu_0\frac{\partial
F}{\partial p}\left(\mu_0,p_0\right)\,p_1,
\end{equation}
and
\begin{equation}\label{MaxVacFull1ninshtrgravortghhghgjkgghklhjgkghghjjkjhjkkggjkhjkhjjhhfhjhkjhjjhkhbbgjlmkmmhzzzyyykkknnnNWBWHWPPNjjlllkkkkkkljkkjjkjk9uu9uuiiuuiiu}
\mu_0\left(\frac{\partial\vec u_1}{\partial t}-curl_{\vec
x}\left\{\vec u_0\times\vec u_1\right\}+(\Div_{\vec x}\vec
u_1)\,\vec u_0\right) \approx-\nabla_{\vec x} p_1.
\end{equation}
Note that, as before, it can be easily proved that
\er{MaxVacFull1ninshtrgravortghhghgjkgghklhjgkghghjjkjhjkkggjkhjkhjjhhfhjhkjhjjhkhbbgjlmkmmhzzzyyykkknnnNWBWHWPPNjjlllkkkkkkljkkjjkjkhjjhhgghgggghjjhgjggh},
\er{MaxVacFull1ninshtrgravortghhghgjkgghklhjnnmnmmngkghghjjkjhjkkggjkhjkhjjhhfhjhkjkhbbgjhzzzyyykkknnnNWBWHWPPNjjhjhjhjhghgjkhkhlkkhhkhkhhhkhhjgjggjkkhhkjhjjhjhjhjhjjggijjhj}
and
\er{MaxVacFull1ninshtrgravortghhghgjkgghklhjgkghghjjkjhjkkggjkhjkhjjhhfhjhkjhjjhkhbbgjlmkmmhzzzyyykkknnnNWBWHWPPNjjlllkkkkkkljkkjjkjk9uu9uuiiuuiiu}
are invariant under the change of inertial or non-inertial cartesian
coordinate system. Thus, as before,
\er{MaxVacFull1ninshtrgravortghhghgjkgghklhjgkghghjjkjhjkkggjkhjkhjjhhfhjhkjhjjhkhbbgjlmkmmhzzzyyykkknnnNWBWHWPPNjjlllkkkkkkljkkjjkjkhjjhhgghgggghjjhgjggh},
\er{MaxVacFull1ninshtrgravortghhghgjkgghklhjnnmnmmngkghghjjkjhjkkggjkhjkhjjhhfhjhkjkhbbgjhzzzyyykkknnnNWBWHWPPNjjhjhjhjhghgjkhkhlkkhhkhkhhhkhhjgjggjkkhhkjhjjhjhjhjhjjggijjhj}
and
\er{MaxVacFull1ninshtrgravortghhghgjkgghklhjgkghghjjkjhjkkggjkhjkhjjhhfhjhkjhjjhkhbbgjlmkmmhzzzyyykkknnnNWBWHWPPNjjlllkkkkkkljkkjjkjk9uu9uuiiuuiiu}
are still valid if
\er{MaxVacFull1yiyiyninshtrgravortghhghgjkgghklhjgkghghjjkjhjkkggjkhjkhjjhhfhjhkjkhbbgjhzzzyyykkknnnNWBWHWPPNjjhjhjhjhplllhjggjl;k}
and
\er{MaxVacFull1yiyiyninshtrgravortghhghgjkgghklhjghjhnkkghghjjkjhjkkggjkhjkhjjhhfhjhkjkhbbgjhzzzyyykkknnnNWBWHWPPNjjhjhjhjhplllhjggjl;k}
hold only in some specific (possibly artificial) inertial or
non-inertial cartesian coordinate system.

In particular, in the case of in the case of the simplest ideal gas
where \er{noninchredPPNbjjhjjkhkhkljjljljklhjhjhjghhgghjhhj} and
\er{noninchredPPNbjjhjjkhkhlkkkklhjgghgh} are valid, by
\er{MaxVacFull1nkkkinshtrgravortghhghgjkgghklhjgkghghjjkjhjkkggjkhjkhjjhhfhjhkjkhbbgjhzzzyyykkknnhjhhnNWBWHWPPNjjhjhjhjhplllhjggjll;;lijiihkhhkgjgj}
we have
\begin{equation}\label{MaxVacFull1nkkkinshtrgravortghhghgjkgghklhjgkghghjjkjhjkkggjkhjkhjjhhfhjhkjkhbbgjhzzzyyykkknnhjhhnNWBWHWPPNjjhjhjhjhplllhjggjll;;lijiihkhhkgjgjhjhjhbcc}
\frac{E}{\mu}=F\left(\mu,p\right):=\gamma_0\,\frac{p}{\mu},
\end{equation}
and therefore, inserting
\er{MaxVacFull1nkkkinshtrgravortghhghgjkgghklhjgkghghjjkjhjkkggjkhjkhjjhhfhjhkjkhbbgjhzzzyyykkknnhjhhnNWBWHWPPNjjhjhjhjhplllhjggjll;;lijiihkhhkgjgjhjhjhbcc}
into
\er{MaxVacFull1ninshtrgravortghhghgjkgghklhjgkghghjjkjhjkkggjkhjkhjjhhfhjhkjhjjhkhbbgjlmkmmhzzzyyykkknnnNWBWHWPPNjjlllkkkkkkljkkjjkjkhjjhhgghgggghjjhgjgghhkhjhjh}
gives:
\begin{equation}\label{MaxVacFull1ninshtrgravortghhghgjkgghklhjgkghghjjkjhjkkggjkhjkhjjhhfhjhkjhjjhkhbbgjlmkmmhzzzyyykkknnnNWBWHWPPNjjlllkkkkkkljkkjjkjkhjjhhgghgggghjjhgjgghhkhjhjhjjggjjgjgjg}
c_0:=\sqrt{\frac{\left(1+\gamma_0\right)p_0}{\gamma_0\mu_0}}.
\end{equation}
Moreover, in this case
\er{MaxVacFull1ninshtrgravortghhghgjkgghklhjnnmnmmngkghghjjkjhjkkggjkhjkhjjhhfhjhkjkhbbgjhzzzyyykkknnnNWBWHWPPNjjhjhjhjhghgjkhkhlkkhhkhkhhhkhhjgjggjkkhhkjhjjhjhjhjhjjgg}
reeds as:
\begin{equation}\label{MaxVacFull1ninshtrgravortghhghgjkgghklhjnnmnmmngkghghjjkjhjkkggjkhjkhjjhhfhjhkjkhbbgjhzzzyyykkknnnNWBWHWPPNjjhjhjhjhghgjkhkhlkkhhkhkhhhkhhjgjggjkkhhkjhjjhjhjhjhjjgghyuyh}
\mu_1\approx\frac{\gamma_0}{1+\gamma_0}\,\frac{\mu_0}{p_0}\,p_1.
\end{equation}

On the other hand, in the case barotropic fluid the pressure $p$ is
a function of the density $\mu$ \underline{only}, i.e.
\begin{equation}\label{MaxVacFull1ninshtrgravortghhghgjkgghklhjnnmnmmngkghghjjkjhjkkggjkhjkhjjhhfhjhkjkhbbgjhzzzyyykkknnnNWBWHWPPNjjhjhjhjhghgjkhkhlkkhhjgggggggggggggggggggggggggggggg}
p=\mathcal{R}(\mu).
\end{equation}
Thus, inserting
\er{MaxVacFull1ninshtrgravortghhghgjkgghklhjgkghghjjkjhjkkggjkhjkhjjhhfhjhkjkhbbgjhzzzyyykkknnnNWBWHWPPNjjhjhjhjhplllhjggjl}
into
\er{MaxVacFull1ninshtrgravortghhghgjkgghklhjnnmnmmngkghghjjkjhjkkggjkhjkhjjhhfhjhkjkhbbgjhzzzyyykkknnnNWBWHWPPNjjhjhjhjhghgjkhkhlkkhhjgggggggggggggggggggggggggggggg}
and using
\er{MaxVacFull1yiyiyninshtrgravortghhghgjkgghklhjgkghghjjkjhjkkggjkhjkhjjhhfhjhkjkhbbgjhzzzyyykkknnnNWBWHWPPNjjhjhjhjhplllhjggjl;k}
we deduce
\begin{equation}\label{MaxVacFull1ninshtrgravortghhghgjkgghklhjnnmnmmngkghghjjkjhjkkggjkhjkhjjhhfhjhkjkhbbgjhzzzyyykkknnnNWBWHWPPNjjhjhjhjhghgjkhkhlkkhhjggggggggggggggggggggggggggggggjggjg}
p_1\approx\frac{d\mathcal{R}}{d\mu}(\mu_0)\,\mu_1.
\end{equation}
Therefore, inserting
\er{MaxVacFull1ninshtrgravortghhghgjkgghklhjnnmnmmngkghghjjkjhjkkggjkhjkhjjhhfhjhkjkhbbgjhzzzyyykkknnnNWBWHWPPNjjhjhjhjhghgjkhkhlkkhhjggggggggggggggggggggggggggggggjggjg}
into
\er{MaxVacFull1ninshtrgravortghhghgjkgghklhjgkghghjjkjhjkkggjkhjkhjjhhfhjhkjhjjhkhbbgjlmkmmhzzzyyykkknnnNWBWHWPPNjjlllkkkkkkljkkjjkjkhjjhhgghgggg}
and using again
\er{MaxVacFull1yiyiyninshtrgravortghhghgjkgghklhjghjhnkkghghjjkjhjkkggjkhjkhjjhhfhjhkjkhbbgjhzzzyyykkknnnNWBWHWPPNjjhjhjhjhplllhjggjl;k}
we finally obtain the analogous to
\er{MaxVacFull1ninshtrgravortghhghgjkgghklhjgkghghjjkjhjkkggjkhjkhjjhhfhjhkjhjjhkhbbgjlmkmmhzzzyyykkknnnNWBWHWPPNjjlllkkkkkkljkkjjkjkhjjhhgghgggghjjhgjggh}
wave equation for the oscillating part of the pressure $p_1$  of the
form:
\begin{equation}\label{MaxVacFull1ninshtrgravortghhghgjkgghklhjgkghghjjkjhjkkggjkhjkhjjhhfhjhkjhjjhkhbbgjlmkmmhzzzyyykkknnnNWBWHWPPNjjlllkkkkkkljkkjjkjkhjjhhgghgggghjjhgjgghjkjkjk}
\frac{\partial}{\partial
t}\left\{\frac{1}{c^2_0}\left(\frac{\partial p_1}{\partial t}+\vec
u_0\cdot\nabla_{\vec x}p_1\right)\right\}+\Div_{\vec
x}\left\{\frac{1}{c^2_0}\left(\frac{\partial p_1}{\partial t}+\vec
u_0\cdot\nabla_{\vec x}p_1\right) \vec u_0\right\}
\approx\Delta_{\vec x} p_1,
\end{equation}
where we denote
\begin{equation}\label{MaxVacFull1ninshtrgravortghhghgjkgghklhjgkghghjjkjhjkkggjkhjkhjjhhfhjhkjhjjhkhbbgjlmkmmhzzzyyykkknnnNWBWHWPPNjjlllkkkkkkljkkjjkjkhjjhhgghgggghjjhgjgghhkhjhjhhkhhhhg}
c_0:=\sqrt{\frac{d\mathcal{R}}{d\mu}(\mu_0)}.
\end{equation}
Moreover, the oscillating parts of the density $\mu_1$ and the
velocity $\vec u_1$ can be found from
\er{MaxVacFull1ninshtrgravortghhghgjkgghklhjnnmnmmngkghghjjkjhjkkggjkhjkhjjhhfhjhkjkhbbgjhzzzyyykkknnnNWBWHWPPNjjhjhjhjhghgjkhkhlkkhhjggggggggggggggggggggggggggggggjggjg}
and
\er{MaxVacFull1ninshtrgravortghhghgjkgghklhjgkghghjjkjhjkkggjkhjkhjjhhfhjhkjhjjhkhbbgjlmkmmhzzzyyykkknnnNWBWHWPPNjjlllkkkkkkljkkjjkjk}
respectively, i.e. we have
\begin{equation}\label{hihihhhoo}
p_1\approx\frac{d\mathcal{R}}{d\mu}(\mu_0)\,\mu_1,
\end{equation}
and
\begin{equation}\label{MaxVacFull1ninshtrgravortghhghgjkgghklhjgkghghjjkjhjkkggjkhjkhjjhhfhjhkjhjjhkhbbgjlmkmmhzzzyyykkknnnNWBWHWPPNjjlllkkkkkkljkkjjkjk9uu9uuiiuuiiujkhjhjjojj}
\mu_0\left(\frac{\partial\vec u_1}{\partial t}-curl_{\vec
x}\left\{\vec u_0\times\vec u_1\right\}+(\Div_{\vec x}\vec
u_1)\,\vec u_0\right) \approx-\nabla_{\vec x} p_1.
\end{equation}
Again note that
\er{MaxVacFull1ninshtrgravortghhghgjkgghklhjgkghghjjkjhjkkggjkhjkhjjhhfhjhkjhjjhkhbbgjlmkmmhzzzyyykkknnnNWBWHWPPNjjlllkkkkkkljkkjjkjkhjjhhgghgggghjjhgjgghjkjkjk},
\er{hihihhhoo} and
\er{MaxVacFull1ninshtrgravortghhghgjkgghklhjgkghghjjkjhjkkggjkhjkhjjhhfhjhkjhjjhkhbbgjlmkmmhzzzyyykkknnnNWBWHWPPNjjlllkkkkkkljkkjjkjk9uu9uuiiuuiiujkhjhjjojj}
are invariant under the change of inertial or non-inertial cartesian
coordinate system. Thus, as before,
\er{MaxVacFull1ninshtrgravortghhghgjkgghklhjgkghghjjkjhjkkggjkhjkhjjhhfhjhkjhjjhkhbbgjlmkmmhzzzyyykkknnnNWBWHWPPNjjlllkkkkkkljkkjjkjkhjjhhgghgggghjjhgjgghjkjkjk},
\er{hihihhhoo} and
\er{MaxVacFull1ninshtrgravortghhghgjkgghklhjgkghghjjkjhjkkggjkhjkhjjhhfhjhkjhjjhkhbbgjlmkmmhzzzyyykkknnnNWBWHWPPNjjlllkkkkkkljkkjjkjk9uu9uuiiuuiiujkhjhjjojj}
are still valid if
\er{MaxVacFull1yiyiyninshtrgravortghhghgjkgghklhjgkghghjjkjhjkkggjkhjkhjjhhfhjhkjkhbbgjhzzzyyykkknnnNWBWHWPPNjjhjhjhjhplllhjggjl;k}
and
\er{MaxVacFull1yiyiyninshtrgravortghhghgjkgghklhjghjhnkkghghjjkjhjkkggjkhjkhjjhhfhjhkjkhbbgjhzzzyyykkknnnNWBWHWPPNjjhjhjhjhplllhjggjl;k}
hold only in some specific (possibly artificial) inertial or
non-inertial cartesian coordinate system.

\subsection{Propagation of waves in the moving elastic
body}\label{gjgfhffhfh}
Consistently with the general equation of
motion of a continuum medium
\er{MaxVacFull1ninshtrgravortghhghgjkgghklhjgkghghjjkjhjkkggjkhjkhjjhhfhjhkjkhbbgjhzzzyyykkknnnNWBWHWPPNjj},
consider the motion of an elastic body:
\begin{multline}\label{MaxVacFull1ninshtrgravortghhghgjkgghklhjgkghghjjkjhjkkggjkhjkhjjhhfhjhkjkhbbgjhzzzyyykkknnnNWBWHWPPNjjjljkjkkhhhhhhhhhhhh}
\mu\left(\frac{\partial\vec u}{\partial t}+d_{\vec x}\vec u\cdot
\vec u\right)=\\-\mu\vec u\times curl_{\vec x}\vec
v+\mu\left(\partial_{t}\vec v+\nabla_{\vec x}\frac{1}{2}|\vec
v|^2\right)+\rho\vec E+\frac{1}{c}\vec j\times\vec B+\Div_{\vec
x}\left\{\alpha\mathcal{E}(\vec
x,t)+\beta\left(tr\left\{\mathcal{E}(\vec
x,t)\right\}\right)I\right\}.
\end{multline}
where
\begin{equation}\label{hjgjgjgkllujkjjjkhhjhljjkjkljjljjhhh}
\mathcal{E}(\vec x,t):=\frac{1}{2}\left(I-\left\{d_{\vec x}\vec
f(\vec x,t)\right\}^T\cdot d_{\vec x}\vec f(\vec x,t)\right),
\end{equation}
consistently with \er{hjgjgjgkllujkjjjkhhjhljjkjkljjljj} and
\er{MaxVacFull1ninshtrgravortghhghgjkgghklhjgkghghjjkjhjkkggjkhjkhjjhhfhjhkhjjkhjgzzzyyykkknnnNWBWHWPPNjjhhjhjkphkhkhhklklkl},
with
\begin{equation}\label{hjgghhghhffhkkmkljjnnnnghgghgh}
\begin{cases}
\frac{\partial\vec f}{\partial t}\left(\vec x,t\right)+d_{\vec
x}\vec f\left(\vec x,t\right)\cdot\vec u\left(\vec x,t\right)=0
\\
\vec f(\vec x,t_0)=\vec x,
\end{cases}
\end{equation}
consistently with \er{hjgghhghhffhkkmkljjnnnn}, where the initial
instant of time $t_0$ is chosen (possibly fictitiously) in such a
way that our elastic body is relaxed at time $t_0$. Moreover, by
\er{hjhjjttuttutukjhkhlkkjkhjljjzz}
we have
\begin{equation}\label{hjhjjttuttutukjhkhlkkjkhjljj}
\mu\left(\vec x,t\right)=\mu(\vec x,t_0)\left|det\{d_{\vec x}\vec
f(\vec x,t)\}\right|.
\end{equation}

Next we assume that
\begin{multline}\label{MaxVacFull1ninshtrgravortghhghgjkgghklhjgkghghjjkjhjkkggjkhjkhjjhhfhjhkjkhbbgjhzzzyyykkknnnNWBWHWPPNjjhjhjhjhplllhjggjlcc}
\vec u(\vec x,t)=\vec u_0(\vec x,t)+\vec u_1(\vec x,t),\quad
\mu(\vec x,t)=\mu_0(\vec x,t)+\mu_1(\vec x,t),\\ \vec f(\vec
x,t)=\vec f_0(\vec x,t)+\vec f_1(\vec x,t)\quad\text{and}\quad
\mathcal{E}(\vec x,t)=\mathcal{E}_{0}(\vec x,t)+\mathcal{E}_{1}(\vec
x,t)
\end{multline}
where $\vec u_0(\vec x,t),\mu_0(\vec x,t),\vec f_0(\vec
x,t),\mathcal{E}_{0}(\vec x,t)$ are the averages of $\vec u(\vec
x,t),\mu(\vec x,t),\vec f(\vec x,t),\mathcal{E}(\vec x,t)$ on small
spatial and temporal intervals, surrounding the point $(\vec x,t)$.
Although we assume these intervals of space and time to be very
small, we also assume them to be quite \underline{macroscopic}. We
call $\vec u_1,\mu_1,\vec f_1,\mathcal{E}_{1}$ the oscillating parts
of $\vec u,\mu,\vec f,\mathcal{E}$ and we assume that they are small
with respect to the averages $\vec u_0,\mu_0,\vec f_0,\mathcal{E}_0$
i.e. we have:
\begin{equation}\label{MaxVacFull1yiyiyninshtrgravortghhghgjkgghklhjgkghghjjkjhjkkggjkhjkhjjhhfhjhkjkhbbgjhzzzyyykkknnnNWBWHWPPNjjhjhjhjhplllhjggjl;kcc}
|\vec u_1|\ll|\vec u_0|,\quad|\mu_1|\ll |\mu_0|,\quad|\vec f_1|\ll
|\vec f_0|.
\end{equation}
However, we assume that $\vec u_1,\mu_1,\vec f_1$ are highly
oscillate and thus they changes spatially and temporary much faster
than the averages  $\vec u_0,\mu_0,\vec f_0$ and the fields $\vec v,
\vec E,\vec B,\rho,\vec j$, i.e. we have:
\begin{multline}\label{MaxVacFull1yiyiyninshtrgravortghhghgjkgghklhjghjhnkkghghjjkjhjkkggjkhjkhjjhhfhjhkjkhbbgjhzzzyyykkknnnNWBWHWPPNjjhjhjhjhplllhjggjl;kcc}
\frac{|d_{\vec x}\alpha|}{|\alpha|}+\frac{|d_{\vec
x}\beta|}{|\beta|}+\frac{\left|d_{\vec x}\left(\rho\vec
E+\frac{1}{c}\vec j\times\vec B\right)\right|}{\left|\rho\vec
E+\frac{1}{c}\vec j\times\vec B\right|}+\frac{|d_{\vec x}\vec
v|}{|\vec v|}+\frac{|d_{\vec x}\mu_0|}{|\mu_0|}+\frac{|d_{\vec
x}\vec f_0|}{|\vec f_0|} +\frac{|d_{\vec x}\vec u_0|}{|\vec u_0|}
\ll \min{\left\{\frac{|d_{\vec x}\vec u_1|}{|\vec u_1|},\frac{|d_{\vec x}\mu_1|}{|\mu_1|},\frac{|d_{\vec x}\vec f_1|}{|\vec f_1|}\right\}},\\
\quad\quad
\frac{\left|\partial_t\left(\rho\vec E+\frac{1}{c}\vec j\times\vec
B\right)\right|}{\left|\rho\vec E+\frac{1}{c}\vec j\times\vec
B\right|}+\frac{|\partial_t\vec v|}{|\vec
v|}+\frac{|\partial_t\mu_0|}{|\mu_0|}+\frac{|\partial_t \vec
f_0|}{|\vec f_0|}+\frac{|\partial_t\vec u_0|}{|\vec u_0|}\ll
\min{\left\{\frac{|\partial_t\vec u_1|}{|\vec
u_1|},\frac{|\partial_t\mu_1|}{|\mu_1|},\frac{|\partial_t \vec
f_1|}{|\vec f_1|}\right\}},\\ \frac{|d^2_{\vec x}\vec f_0|}{|d_{\vec
x}\vec f_0|}\ll\frac{|d^2_{\vec x}\vec f_1|}{|d_{\vec x}\vec f_1|}
\quad\quad \frac{|\vec u_1|}{|\vec u_0|}\ll
\min{\left\{\frac{|d_{\vec x}\vec u_1|}{|d_{\vec x}\vec
u_0|},\frac{|d_{\vec x}\vec u_1|}{|d_{\vec x}\vec v|}\right\}},
\quad\quad\text{and} \quad\quad
\frac{|\mu_1|}{|\mu_0|}\ll\frac{|d_{\vec x}\vec u_1||\vec
u_0||\mu_0|}{\left|\rho\vec E+\frac{1}{c}\vec j\times\vec B\right|}.
\end{multline}
Finally, we assume that the fields $\vec v, \vec E,\vec B,\rho,\vec
j$ change slowly with respect to the oscillations of  $\vec
u_1,\mu_1,\vec f_1$ and thus we assume that $\vec v, \vec E,\vec
B,\rho,\vec j$ can be replaced by their spatial and temporal
averages. Note that $\mu,\mu_0,\mu_1$ behave like \underline{proper
scalar} fields, $\mathcal{E},\mathcal{E}_0,\mathcal{E}_1$ behave
like \underline{proper matrix} fields and $\vec u,\vec u_0$ behave
like \underline{speed-like vector}  fields under the change of
cartesian coordinate systems. Thus, since $\vec u_1=\vec u-\vec
u_0$, we deduce that $\vec u_1$ behaves like a \underline{proper
vector} field  under the change of cartesian coordinate systems.
Furthermore, using \er{hjgjgjgkllujkjjjkhhjhljjk} we deduce that the
vector field $\vec g_1\left(\vec x,t\right)$ defined by the
following:
\begin{equation}\label{jgjgghghghhkhkhkh}
\vec g_1\left(\vec x,t\right):=\left(d_{\vec x}\vec f_0\left(\vec
x,t\right)\right)^{-1}\cdot\vec f_1\left(\vec x,t\right),
\end{equation}
behaves like a \underline{proper vector} field. Finally, note that
obviously the averages of $\vec u_1,\mu_1,\vec f_1,\mathcal{E}_1$
vanish.

Next we would like to approximate
\er{hjgghhghhffhkkmkljjnnnnghgghgh} and
\er{MaxVacFull1ninshtrgravortghhghgjkgghklhjgkghghjjkjhjkkggjkhjkhjjhhfhjhkjkhbbgjhzzzyyykkknnnNWBWHWPPNjjjljkjkkhhhhhhhhhhhh}.
In the rough level of approximation we could just replace $\vec
u,\mu,\vec f$ by  $\vec u_0,\mu_0,\vec f_0$. However, we would like
to use a more delicate approximation, taking into the account the
first order terms with $\vec u_1,\mu_1,\vec f_1$. Then, as before,
by \er{hjgghhghhffhkkmkljjnnnnghgghgh} and
\er{MaxVacFull1ninshtrgravortghhghgjkgghklhjgkghghjjkjhjkkggjkhjkhjjhhfhjhkjkhbbgjhzzzyyykkknnnNWBWHWPPNjjhjhjhjhplllhjggjlcc}
and using
\er{MaxVacFull1yiyiyninshtrgravortghhghgjkgghklhjgkghghjjkjhjkkggjkhjkhjjhhfhjhkjkhbbgjhzzzyyykkknnnNWBWHWPPNjjhjhjhjhplllhjggjl;kcc}
we have:
\begin{equation}\label{hjgghhghhffhkkmkljjnnnnghgghghjhyjyy}
\frac{\partial\vec f_0}{\partial t}\left(\vec
x,t\right)+\frac{\partial\vec f_1}{\partial t}\left(\vec
x,t\right)+d_{\vec x}\vec f_0\left(\vec x,t\right)\cdot\vec
u_0\left(\vec x,t\right)+d_{\vec x}\vec f_0\left(\vec
x,t\right)\cdot\vec u_1\left(\vec x,t\right)+d_{\vec x}\vec
f_1\left(\vec x,t\right)\cdot\vec u_0\left(\vec x,t\right)\approx 0.
\end{equation}
Thus, averaging of \er{hjgghhghhffhkkmkljjnnnnghgghghjhyjyy} gives
\begin{equation}\label{hjgghhghhffhkkmkljjnnnnghgghghjhyjyyhh}
\frac{\partial\vec f_0}{\partial t}\left(\vec x,t\right)+d_{\vec
x}\vec f_0\left(\vec x,t\right)\cdot\vec u_0\left(\vec
x,t\right)\approx 0.
\end{equation}
Therefore, subtracting \er{hjgghhghhffhkkmkljjnnnnghgghghjhyjyyhh}
from unaveraged \er{hjgghhghhffhkkmkljjnnnnghgghghjhyjyy} we obtain
\begin{equation}\label{hjgghhghhffhkkmkljjnnnnghgghghjhyjyyjhhjhgjgj}
\frac{\partial\vec f_1}{\partial t}\left(\vec x,t\right)+d_{\vec
x}\vec f_1\left(\vec x,t\right)\cdot\vec u_0\left(\vec
x,t\right)+d_{\vec x}\vec f_0\left(\vec x,t\right)\cdot\vec
u_1\left(\vec x,t\right)\approx 0.
\end{equation}
Then, by \er{hjgghhghhffhkkmkljjnnnnghgghghjhyjyyjhhjhgjgj} and
\er{MaxVacFull1yiyiyninshtrgravortghhghgjkgghklhjghjhnkkghghjjkjhjkkggjkhjkhjjhhfhjhkjkhbbgjhzzzyyykkknnnNWBWHWPPNjjhjhjhjhplllhjggjl;kcc}
we deduce:
\begin{equation}\label{hjgghhghhffhkkmkljjnnnnghgghghjhyjyyjhhjhgjgjjhjhjhkhhkhjh}
\frac{\partial}{\partial t}\left(\left(d_{\vec x}\vec f_0\left(\vec
x,t\right)\right)^{-1}\cdot\vec f_1\left(\vec
x,t\right)\right)+d_{\vec x}\left(\left(d_{\vec x}\vec f_0\left(\vec
x,t\right)\right)^{-1}\cdot\vec f_1\left(\vec
x,t\right)\right)\cdot\vec u_0\left(\vec x,t\right)\approx-\vec
u_1\left(\vec x,t\right),
\end{equation}
i.e.
\begin{equation}\label{hjgghhghhffhkkmkljjnnnnghgghghjhyjyyjhhjhgjgjjhjhjh}
\frac{\partial\vec g_1}{\partial t}+d_{\vec x}\vec g_1\cdot\vec
u_0\approx-\vec u_1,
\end{equation}
where $\vec g_1(\vec x,t)$ is given by \er{jgjgghghghhkhkhkh}. Note
that, by
\er{MaxVacFull1yiyiyninshtrgravortghhghgjkgghklhjghjhnkkghghjjkjhjkkggjkhjkhjjhhfhjhkjkhbbgjhzzzyyykkknnnNWBWHWPPNjjhjhjhjhplllhjggjl;kcc},
\er{hjgghhghhffhkkmkljjnnnnghgghghjhyjyyjhhjhgjgjjhjhjh} can be
written in the alternative approximate form
\begin{equation}\label{hjgghhghhffhkkmkljjnnnnghgghghjhyjyyjhhjhgjgjjhjhjhkjhk}
\frac{\partial\vec g_1}{\partial t}+\left(\Div_{\vec x}\vec
u_0\right)\vec g_1+d_{\vec x}\vec g_1\cdot\vec u_0-d_{\vec x}\vec
u_0\cdot\vec g_1\approx-\vec u_1,
\end{equation}
and thus, by
\er{hjgghhghhffhkkmkljjnnnnghgghghjhyjyyjhhjhgjgjjhjhjhkjhk} and
\er{apfrm6} we have:
\begin{equation}\label{hjgghhghhffhkkmkljjnnnnghgghghjhyjyyjhhjhgjgjjhjhjhjhggjghjhnhmnmmbjoujhk}
\frac{\partial\vec g_1}{\partial t}-curl_{\vec x}\left\{\vec
u_0\times\vec g_1\right\}+\left(\Div_{\vec x}\vec g_1\right)\vec
u_0\approx-\vec u_1.
\end{equation}
On the other hand, as before, by
\er{MaxVacFull1ninshtrgravortghhghgjkgghklhjgkghghjjkjhjkkggjkhjkhjjhhfhjhkjkhbbgjhzzzyyykkknnnNWBWHWPPNjjhjhjhjhplllhjggjlcc}
and using
\er{MaxVacFull1yiyiyninshtrgravortghhghgjkgghklhjgkghghjjkjhjkkggjkhjkhjjhhfhjhkjkhbbgjhzzzyyykkknnnNWBWHWPPNjjhjhjhjhplllhjggjl;kcc}
we rewrite
\er{MaxVacFull1ninshtrgravortghhghgjkgghklhjgkghghjjkjhjkkggjkhjkhjjhhfhjhkjkhbbgjhzzzyyykkknnnNWBWHWPPNjjjljkjkkhhhhhhhhhhhh}
as:
\begin{multline}\label{MaxVacFull1ninshtrgravortghhghgjkgghklhjgkghghjjkjhjkkggjkhjkhjjhhfhjhkjkhbbgjhzzzyyykkknnnNWBWHWPPNjjjljkjkkhhhhhhhhhhhhhjjh}
\mu_0\left(\frac{\partial\vec u_0}{\partial t}+d_{\vec x}\vec
u_0\cdot \vec u_0\right)+\mu_1\left(\frac{\partial\vec u_0}{\partial
t}+d_{\vec x}\vec u_0\cdot \vec
u_0\right)+\mu_0\left(\frac{\partial\vec u_1}{\partial t}+d_{\vec
x}\vec u_1\cdot \vec u_0+d_{\vec x}\vec u_0\cdot \vec
u_1\right)\approx\\-\mu_0\vec u_0\times curl_{\vec x}\vec
v-\mu_1\vec u_0\times curl_{\vec x}\vec v-\mu_0\vec u_1\times
curl_{\vec x}\vec v+\mu_0\left(\partial_{t}\vec v+\nabla_{\vec
x}\frac{1}{2}|\vec v|^2\right)+\mu_1\left(\partial_{t}\vec
v+\nabla_{\vec x}\frac{1}{2}|\vec v|^2\right)\\+\rho\vec
E+\frac{1}{c}\vec j\times\vec B+\Div_{\vec
x}\left\{\alpha\mathcal{E}_0+\beta\left(tr\,\mathcal{E}_0\right)I\right\}+\Div_{\vec
x}\left\{\alpha\mathcal{E}_1+\beta\left(tr\,\mathcal{E}_1\right)I\right\}.
\end{multline}
Thus, as before, averaging
\er{MaxVacFull1ninshtrgravortghhghgjkgghklhjgkghghjjkjhjkkggjkhjkhjjhhfhjhkjkhbbgjhzzzyyykkknnnNWBWHWPPNjjjljkjkkhhhhhhhhhhhhhjjh}
gives
\begin{multline}\label{MaxVacFull1ninshtrgravortghhghgjkgghklhjgkghghjjkjhjkkggjkhjkhjjhhfhjhkjkhbbgjhzzzyyykkknnnNWBWHWPPNjjjljkjkkhhhhhhhhhhhhhjjhhkjhjh}
\mu_0\left(\frac{\partial\vec u_0}{\partial t}+d_{\vec x}\vec
u_0\cdot \vec u_0\right)\approx-\mu_0\vec u_0\times curl_{\vec
x}\vec v+\mu_0\left(\partial_{t}\vec v+\nabla_{\vec
x}\frac{1}{2}|\vec v|^2\right)+\rho\vec E+\frac{1}{c}\vec
j\times\vec B\\+\Div_{\vec
x}\left\{\alpha\mathcal{E}_0+\beta\left(tr\,\mathcal{E}_0\right)I\right\}.
\end{multline}
Therefore, subtracting
\er{MaxVacFull1ninshtrgravortghhghgjkgghklhjgkghghjjkjhjkkggjkhjkhjjhhfhjhkjkhbbgjhzzzyyykkknnnNWBWHWPPNjjjljkjkkhhhhhhhhhhhhhjjhhkjhjh}
from unaveraged
\er{MaxVacFull1ninshtrgravortghhghgjkgghklhjgkghghjjkjhjkkggjkhjkhjjhhfhjhkjkhbbgjhzzzyyykkknnnNWBWHWPPNjjjljkjkkhhhhhhhhhhhhhjjh}
we obtain
\begin{multline}\label{MaxVacFull1ninshtrgravortghhghgjkgghklhjgkghghjjkjhjkkggjkhjkhjjhhfhjhkjkhbbgjhzzzyyykkknnnNWBWHWPPNjjjljkjkkhhhhhhhhhhhhhjjhjhjggj}
\mu_1\left(\frac{\partial\vec u_0}{\partial t}+d_{\vec x}\vec
u_0\cdot \vec u_0\right)+\mu_0\left(\frac{\partial\vec u_1}{\partial
t}+d_{\vec x}\vec u_1\cdot \vec u_0+d_{\vec x}\vec u_0\cdot \vec
u_1\right)\approx\\-\mu_1\vec u_0\times curl_{\vec x}\vec
v-\mu_0\vec u_1\times curl_{\vec x}\vec
v+\mu_1\left(\partial_{t}\vec v+\nabla_{\vec x}\frac{1}{2}|\vec
v|^2\right)+\Div_{\vec
x}\left\{\alpha\mathcal{E}_1+\beta\left(tr\,\mathcal{E}_1\right)I\right\}.
\end{multline}
Thus, by
\er{MaxVacFull1ninshtrgravortghhghgjkgghklhjgkghghjjkjhjkkggjkhjkhjjhhfhjhkjkhbbgjhzzzyyykkknnnNWBWHWPPNjjjljkjkkhhhhhhhhhhhhhjjhjhjggj}
and
\er{MaxVacFull1ninshtrgravortghhghgjkgghklhjgkghghjjkjhjkkggjkhjkhjjhhfhjhkjkhbbgjhzzzyyykkknnnNWBWHWPPNjjjljkjkkhhhhhhhhhhhhhjjhhkjhjh}
we infer
\begin{multline}\label{MaxVacFull1ninshtrgravortghhghgjkgghklhjgkghghjjkjhjkkggjkhjkhjjhhfhjhkjkhbbgjhzzzyyykkknnnNWBWHWPPNjjjljkjkkhhhhhhhhhhhhhjjhjhjggjkjkhhk}
\mu_0\left(\frac{\partial\vec u_1}{\partial t}+d_{\vec x}\vec
u_1\cdot \vec u_0+d_{\vec x}\vec u_0\cdot \vec
u_1\right)\approx-\mu_0\vec u_1\times curl_{\vec x}\vec v+\Div_{\vec
x}\left\{\alpha\mathcal{E}_1+\beta\left(tr,\mathcal{E}_1\right)I\right\}\\-\frac{\mu_1}{\mu_0}\,\Div_{\vec
x}\left\{\alpha\mathcal{E}_0+\beta\left(tr\,\mathcal{E}_0\right)I\right\}-\frac{\mu_1}{\mu_0}\left(\rho\vec
E+\frac{1}{c}\vec j\times\vec B\right).
\end{multline}
Therefore, using
\er{MaxVacFull1yiyiyninshtrgravortghhghgjkgghklhjghjhnkkghghjjkjhjkkggjkhjkhjjhhfhjhkjkhbbgjhzzzyyykkknnnNWBWHWPPNjjhjhjhjhplllhjggjl;kcc}
we further approximate
\er{MaxVacFull1ninshtrgravortghhghgjkgghklhjgkghghjjkjhjkkggjkhjkhjjhhfhjhkjkhbbgjhzzzyyykkknnnNWBWHWPPNjjjljkjkkhhhhhhhhhhhhhjjhjhjggjkjkhhk}
as:
\begin{equation}\label{MaxVacFull1ninshtrgravortghhghgjkgghklhjgkghghjjkjhjkkggjkhjkhjjhhfhjhkjkhbbgjhzzzyyykkknnnNWBWHWPPNjjjljkjkkhhhhhhhhhhhhfnfffhuyuyuyutkhhkkh}
\mu_0\left(\frac{\partial\vec u_1}{\partial t}+d_{\vec x}\vec
u_1\cdot \vec u_0\right)\approx \Div_{\vec
x}\left\{\alpha\mathcal{E}_1+\beta\left(tr\,\mathcal{E}_1\right)I\right\}-\frac{\mu_1}{\mu_0}\,\Div_{\vec
x}\left\{\alpha\mathcal{E}_0+\beta\left(tr\,\mathcal{E}_0\right)I\right\}.
\end{equation}
On the other hand, by \er{hjhjjttuttutukjhkhlkkjkhjljj},
\er{MaxVacFull1ninshtrgravortghhghgjkgghklhjgkghghjjkjhjkkggjkhjkhjjhhfhjhkjkhbbgjhzzzyyykkknnnNWBWHWPPNjjhjhjhjhplllhjggjlcc},
\er{hjgjgjgkllujkjjjkhhjhljjkjkljjljjhhh} and
\er{MaxVacFull1yiyiyninshtrgravortghhghgjkgghklhjghjhnkkghghjjkjhjkkggjkhjkhjjhhfhjhkjkhbbgjhzzzyyykkknnnNWBWHWPPNjjhjhjhjhplllhjggjl;kcc}
we deduce
\begin{equation}\label{MaxVacFull1ninshtrgravortghhghgjkgghklhjgkghghjjkjhjkkggjkhjkhjjhhfhjhkjkhbbgjhzzzyyykkknnnNWBWHWPPNjjjljkjkkhhhhhhhhhhhhfnfffhuyuyuyutkhhkkhhkhhhh}
\left| \frac{\mu_1}{\mu_0}\,\Div_{\vec
x}\left\{\alpha\mathcal{E}_0+\beta\left(tr\,\mathcal{E}_0\right)I\right\}\right|\ll\left|\Div_{\vec
x}\left\{\alpha\mathcal{E}_1+\beta\left(tr\,\mathcal{E}_1\right)I\right\}\right|.
\end{equation}
Thus, by
\er{MaxVacFull1ninshtrgravortghhghgjkgghklhjgkghghjjkjhjkkggjkhjkhjjhhfhjhkjkhbbgjhzzzyyykkknnnNWBWHWPPNjjjljkjkkhhhhhhhhhhhhfnfffhuyuyuyutkhhkkhhkhhhh}
we further approximate
\er{MaxVacFull1ninshtrgravortghhghgjkgghklhjgkghghjjkjhjkkggjkhjkhjjhhfhjhkjkhbbgjhzzzyyykkknnnNWBWHWPPNjjjljkjkkhhhhhhhhhhhhfnfffhuyuyuyutkhhkkh}
as:
\begin{equation}\label{MaxVacFull1ninshtrgravortghhghgjkgghklhjgkghghjjkjhjkkggjkhjkhjjhhfhjhkjkhbbgjhzzzyyykkknnnNWBWHWPPNjjjljkjkkhhhhhhhhhhhhfnfffhuyuyuyut}
\mu_0\left(\frac{\partial\vec u_1}{\partial t}+d_{\vec x}\vec
u_1\cdot \vec u_0\right)\approx\Div_{\vec
x}\left\{\alpha\mathcal{E}_1+\beta\left(tr\,\mathcal{E}_1\right)I\right\}.
\end{equation}
Therefore, by inserting
\er{hjgghhghhffhkkmkljjnnnnghgghghjhyjyyjhhjhgjgjjhjhjh} into
\er{MaxVacFull1ninshtrgravortghhghgjkgghklhjgkghghjjkjhjkkggjkhjkhjjhhfhjhkjkhbbgjhzzzyyykkknnnNWBWHWPPNjjjljkjkkhhhhhhhhhhhhfnfffhuyuyuyut}
we deduce:
\begin{multline}\label{MaxVacFull1ninshtrgravortghhghgjkgghklhjgkghghjjkjhjkkggjkhjkhjjhhfhjhkjkhbbgjhzzzyyykkknnnNWBWHWPPNjjjljkjkkhhhhhhhhhhhhfnfffhuyuyuyutjhjhjh}
\mu_0\left(\frac{\partial}{\partial t}\left(\frac{\partial\vec
g_1}{\partial t}+d_{\vec x}\vec g_1\cdot\vec u_0\right)+d_{\vec
x}\left(\frac{\partial\vec g_1}{\partial t}+d_{\vec x}\vec
g_1\cdot\vec u_0\right)\cdot \vec u_0\right) \approx-\Div_{\vec
x}\left\{\alpha\mathcal{E}_1+\beta\left(tr\left\{\mathcal{E}_1\right\}\right)I\right\},
\end{multline}
where $\vec g_1(\vec x,t)$ is given by \er{jgjgghghghhkhkhkh}. On
the other hand, by \er{hjgjgjgkllujkjjjkhhjhljjkjkljjljjhhh},
\er{MaxVacFull1ninshtrgravortghhghgjkgghklhjgkghghjjkjhjkkggjkhjkhjjhhfhjhkjkhbbgjhzzzyyykkknnnNWBWHWPPNjjhjhjhjhplllhjggjlcc},
\er{MaxVacFull1yiyiyninshtrgravortghhghgjkgghklhjgkghghjjkjhjkkggjkhjkhjjhhfhjhkjkhbbgjhzzzyyykkknnnNWBWHWPPNjjhjhjhjhplllhjggjl;kcc}
and
\er{MaxVacFull1yiyiyninshtrgravortghhghgjkgghklhjghjhnkkghghjjkjhjkkggjkhjkhjjhhfhjhkjkhbbgjhzzzyyykkknnnNWBWHWPPNjjhjhjhjhplllhjggjl;kcc}
we have
\begin{equation}\label{hjgjgjgkllujkjjjkhhjhljjkjkljjljjhhhjkhhugg}
\mathcal{E}_0(\vec x,t)\approx\frac{1}{2}\left(I-\left\{d_{\vec
x}\vec f_0(\vec x,t)\right\}^T\cdot d_{\vec x}\vec f_0(\vec
x,t)\right),
\end{equation}
and
\begin{multline}\label{hjgjgjgkllujkjjjkhhjhljjkjkljjljjhkhjhhffgjjjljjlmbaa}
\mathcal{E}_{1}(\vec x,t)\approx-\frac{1}{2}\left(\left\{d_{\vec
x}\vec f_1(\vec x,t)\right\}^T\cdot d_{\vec x}\vec f_0(\vec
x,t)+\left\{d_{\vec x}\vec f_0(\vec x,t)\right\}^T\cdot d_{\vec
x}\vec f_1(\vec x,t)\right)\approx\\
-\frac{1}{2}\left\{d_{\vec x}\vec f_0(\vec x,t)\right\}^T\cdot
d_{\vec x}\vec f_0(\vec x,t)\cdot d_{\vec x}\left(\left(d_{\vec
x}\vec f_0(\vec x,t)\right)^{-1}\cdot\vec f_1(\vec
x,t)\right)\\-\frac{1}{2}\left\{ d_{\vec x}\left(\left(d_{\vec
x}\vec f_0(\vec x,t)\right)^{-1}\cdot\vec f_1(\vec
x,t)\right)\right\}^T\cdot\left\{d_{\vec x}\vec f_0(\vec
x,t)\right\}^{T}\cdot d_{\vec x}\vec f_0(\vec x,t)\\=
\mathcal{E}_0(\vec x,t)\cdot d_{\vec x}\vec g_1(\vec x,t)+\left\{
d_{\vec x}\vec g_1(\vec x,t)\right\}^T\cdot\mathcal{E}_0(\vec
x,t)-\frac{1}{2}\left(d_{\vec x}\vec g_1(\vec x,t)+\left\{ d_{\vec
x}\vec g_1(\vec x,t)\right\}^T\right),
\end{multline}
where the last equality of
\er{hjgjgjgkllujkjjjkhhjhljjkjkljjljjhkhjhhffgjjjljjlmbaa} follows
from \er{jgjgghghghhkhkhkh} and
\er{hjgjgjgkllujkjjjkhhjhljjkjkljjljjhhhjkhhugg}. So,
\begin{multline}\label{hjgjgjgkllujkjjjkhhjhljjkjkljjljjhkhjhhffgjjjljjlmb}
\mathcal{E}_{1}(\vec x,t)\approx \mathcal{E}_0(\vec x,t)\cdot
d_{\vec x}\vec g_1(\vec x,t)+\left\{ d_{\vec x}\vec g_1(\vec
x,t)\right\}^T\cdot\mathcal{E}_0(\vec x,t)-\frac{1}{2}\left(d_{\vec
x}\vec g_1(\vec x,t)+\left\{ d_{\vec x}\vec g_1(\vec
x,t)\right\}^T\right).
\end{multline}
Thus, inserting
\er{hjgjgjgkllujkjjjkhhjhljjkjkljjljjhkhjhhffgjjjljjlmb} into
\er{MaxVacFull1ninshtrgravortghhghgjkgghklhjgkghghjjkjhjkkggjkhjkhjjhhfhjhkjkhbbgjhzzzyyykkknnnNWBWHWPPNjjjljkjkkhhhhhhhhhhhhfnfffhuyuyuyutjhjhjh}
and using
\er{MaxVacFull1yiyiyninshtrgravortghhghgjkgghklhjghjhnkkghghjjkjhjkkggjkhjkhjjhhfhjhkjkhbbgjhzzzyyykkknnnNWBWHWPPNjjhjhjhjhplllhjggjl;kcc}
finally gives:
\begin{multline}\label{MaxVacFull1ninshtrgravortghhghgjkgghklhjgkghghjjkjhjkkggjkhjkhjjhhfhjhkjkhbbgjhzzzyyykkknnnNWBWHWPPNjjjljkjkkhhhhhhhhhhhhfnfffhuyuyuyutjhjhjhggjjfgfhfh}
\mu_0\left(\frac{\partial}{\partial t}\left(\frac{\partial\vec
g_1}{\partial t}+d_{\vec x}\vec g_1\cdot\vec u_0\right)+d_{\vec
x}\left(\frac{\partial\vec g_1}{\partial t}+d_{\vec x}\vec g_1\cdot\vec u_0\right)\cdot \vec u_0\right)\approx\\
-\alpha\,\Div_{\vec x}\left\{\mathcal{E}_0\cdot d_{\vec x}\vec
g_1\right\}-\alpha\,\Div_{\vec x}\left\{\left\{ d_{\vec x}\vec
g_1\right\}^T\cdot\mathcal{E}_0\right\}-\beta\,\nabla_{\vec
x}\left(tr\left(\mathcal{E}_0\cdot d_{\vec x}\vec
g_1\right)\right)-\beta\,\nabla_{\vec x}\left(tr\left(\left\{
d_{\vec x}\vec
g_1\right\}^T\cdot\mathcal{E}_0\right)\right)\\+\frac{\alpha}{2}\,\Delta_{\vec
x}\vec g_1+\left(\frac{\alpha}{2}+\beta\right)\nabla_{\vec
x}\left(\Div_{\vec x}\vec g_1\right).
\end{multline}
Note again that similarly as we proved
\er{hjgghhghhffhkkmkljjnnnnghgghghjhyjyyjhhjhgjgjjhjhjhjhggjghjhnhmnmmbjoujhk},
i.e.:
\begin{equation}\label{hjgghhghhffhkkmkljjnnnnghgghghjhyjyyjhhjhgjgjjhjhjhjhggjghjhnhmnmmbjoujhkhhh}
\frac{\partial\vec g_1}{\partial t}-curl_{\vec x}\left\{\vec
u_0\times\vec g_1\right\}+\left(\Div_{\vec x}\vec g_1\right)\vec
u_0\approx-\vec u_1.
\end{equation}
we can derive the approximation of
\er{MaxVacFull1ninshtrgravortghhghgjkgghklhjgkghghjjkjhjkkggjkhjkhjjhhfhjhkjkhbbgjhzzzyyykkknnnNWBWHWPPNjjjljkjkkhhhhhhhhhhhhfnfffhuyuyuyutjhjhjhggjjfgfhfh}
as
\begin{multline}\label{MaxVacFull1ninshtrgravortghhghgjkgghklhjgkghghjjkjhjkkggjkhjkhjjhhfhjhkjkhbbgjhzzzyyykkknnnNWBWHWPPNjjjljkjkkhhhhhhhhhhhhfnfffhuyuyuyutjhjhjhggjjfgfhfhkhhjhjh}
\mu_0\frac{\partial}{\partial t}\left(\frac{\partial\vec
g_1}{\partial t}-curl_{\vec x}\left\{\vec u_0\times\vec
g_1\right\}+\left(\Div_{\vec x}\vec g_1\right)\vec
u_0\right)\\-\mu_0\,curl_{\vec x}\left\{\vec
u_0\times\left(\frac{\partial\vec g_1}{\partial t}-curl_{\vec
x}\left\{\vec u_0\times\vec g_1\right\}+\left(\Div_{\vec x}\vec
g_1\right)\vec u_0\right)\right\}\\+\left(\Div_{\vec
x}\left\{\frac{\partial\vec g_1}{\partial t}-curl_{\vec
x}\left\{\vec u_0\times\vec g_1\right\}+\left(\Div_{\vec x}\vec
g_1\right)\vec u_0\right\}\right)\vec u_0
\approx\frac{\alpha}{2}\,\Delta_{\vec x}\vec
g_1+\left(\frac{\alpha}{2}+\beta\right)\nabla_{\vec
x}\left(\Div_{\vec x}\vec g_1\right)\\
-\alpha\,\Div_{\vec x}\left\{\mathcal{E}_0\cdot d_{\vec x}\vec
g_1\right\}-\alpha\,\Div_{\vec x}\left\{\left\{ d_{\vec x}\vec
g_1\right\}^T\cdot\mathcal{E}_0\right\}-\beta\,\nabla_{\vec
x}\left(tr\left(\mathcal{E}_0\cdot d_{\vec x}\vec
g_1\right)\right)-\beta\,\nabla_{\vec x}\left(tr\left(\left\{
d_{\vec x}\vec g_1\right\}^T\cdot\mathcal{E}_0\right)\right).
\end{multline}
In, particular, if our elastic body is \underline{nearly rigid} we
have
\begin{equation}\label{hjgjgjgkllujkjjjkhhjhljjkjkljjljjhhhjkhhugghjjhjhgj}
\mathcal{E}_0\approx 0,
\end{equation}
and we simplify
\er{MaxVacFull1ninshtrgravortghhghgjkgghklhjgkghghjjkjhjkkggjkhjkhjjhhfhjhkjkhbbgjhzzzyyykkknnnNWBWHWPPNjjjljkjkkhhhhhhhhhhhhfnfffhuyuyuyutjhjhjhggjjfgfhfh}
as:
\begin{multline}\label{MaxVacFull1ninshtrgravortghhghgjkgghklhjgkghghjjkjhjkkggjkhjkhjjhhfhjhkjkhbbgjhzzzyyykkknnnNWBWHWPPNjjjljkjkkhhhhhhhhhhhhfnfffhuyuyuyutjhjhjhggjjfgfhfhjh}
\mu_0\left(\frac{\partial}{\partial t}\left(\frac{\partial\vec
g_1}{\partial t}+d_{\vec x}\vec g_1\cdot\vec u_0\right)+d_{\vec
x}\left(\frac{\partial\vec g_1}{\partial t}+d_{\vec x}\vec
g_1\cdot\vec u_0\right)\cdot \vec u_0\right)
\approx\frac{\alpha}{2}\,\Delta_{\vec x}\vec
g_1+\left(\frac{\alpha}{2}+\beta\right)\nabla_{\vec
x}\left(\Div_{\vec x}\vec g_1\right),
\end{multline}
and \er{hjgjgjgkllujkjjjkhhjhljjkjkljjljjhkhjhhffgjjjljjlmb} as:
\begin{equation}\label{hjgjgjgkllujkjjjkhhjhljjkjkljjljjhkhjhhffgjjjljjlmbhkhkhkhk}
\mathcal{E}_{1}\approx-\frac{1}{2}\left(d_{\vec x}\vec g_1+\left\{
d_{\vec x}\vec g_1\right\}^T\right).
\end{equation}
Moreover, we rewrite
\er{MaxVacFull1ninshtrgravortghhghgjkgghklhjgkghghjjkjhjkkggjkhjkhjjhhfhjhkjkhbbgjhzzzyyykkknnnNWBWHWPPNjjjljkjkkhhhhhhhhhhhhfnfffhuyuyuyutjhjhjhggjjfgfhfhkhhjhjh}
as:
\begin{multline}\label{MaxVacFull1ninshtrgravortghhghgjkgghklhjgkghghjjkjhjkkggjkhjkhjjhhfhjhkjkhbbgjhzzzyyykkknnnNWBWHWPPNjjjljkjkkhhhhhhhhhhhhfnfffhuyuyuyutjhjhjhggjjfgfhfhjhhgjgj}
\mu_0\frac{\partial}{\partial t}\left(\frac{\partial\vec
g_1}{\partial t}-curl_{\vec x}\left\{\vec u_0\times\vec
g_1\right\}+\left(\Div_{\vec x}\vec g_1\right)\vec
u_0\right)\\-\mu_0\,curl_{\vec x}\left\{\vec
u_0\times\left(\frac{\partial\vec g_1}{\partial t}-curl_{\vec
x}\left\{\vec u_0\times\vec g_1\right\}+\left(\Div_{\vec x}\vec
g_1\right)\vec u_0\right)\right\}\\+\mu_0\left(\Div_{\vec
x}\left\{\frac{\partial\vec g_1}{\partial t}-curl_{\vec
x}\left\{\vec u_0\times\vec g_1\right\}+\left(\Div_{\vec x}\vec
g_1\right)\vec u_0\right\}\right)\vec u_0
\approx\frac{\alpha}{2}\,\Delta_{\vec x}\vec
g_1+\left(\frac{\alpha}{2}+\beta\right)\nabla_{\vec
x}\left(\Div_{\vec x}\vec g_1\right).
\end{multline}
Next, note that, since $\mu,\mu_0,\mu_1$ are \underline{proper
scalar} fields, $\mathcal{E},\mathcal{E}_0,\mathcal{E}_1$ are
\underline{proper matrix} fields, $\vec u,\vec u_0$ are
\underline{speed-like vector} fields, $\vec u_1$ is a
\underline{proper vector} field and the vector field $\vec
g_1\left(\vec x,t\right)$ defined by \er{jgjgghghghhkhkhkh} is a
\underline{proper vector} field, as before, it can be easily proved
that
\er{MaxVacFull1ninshtrgravortghhghgjkgghklhjgkghghjjkjhjkkggjkhjkhjjhhfhjhkjkhbbgjhzzzyyykkknnnNWBWHWPPNjjjljkjkkhhhhhhhhhhhhfnfffhuyuyuyutjhjhjhggjjfgfhfhkhhjhjh},
\er{hjgghhghhffhkkmkljjnnnnghgghghjhyjyyjhhjhgjgjjhjhjhjhggjghjhnhmnmmbjoujhkhhh},
\er{hjgjgjgkllujkjjjkhhjhljjkjkljjljjhhhjkhhugg},
\er{hjgjgjgkllujkjjjkhhjhljjkjkljjljjhkhjhhffgjjjljjlmb},
\er{MaxVacFull1ninshtrgravortghhghgjkgghklhjgkghghjjkjhjkkggjkhjkhjjhhfhjhkjkhbbgjhzzzyyykkknnnNWBWHWPPNjjjljkjkkhhhhhhhhhhhhfnfffhuyuyuyutjhjhjhggjjfgfhfhjhhgjgj}
\er{hjgjgjgkllujkjjjkhhjhljjkjkljjljjhhhjkhhugghjjhjhgj} and
\er{hjgjgjgkllujkjjjkhhjhljjkjkljjljjhkhjhhffgjjjljjlmbhkhkhkhk} are
invariant under the change of inertial or non-inertial cartesian
coordinate system. Thus,
\er{MaxVacFull1ninshtrgravortghhghgjkgghklhjgkghghjjkjhjkkggjkhjkhjjhhfhjhkjkhbbgjhzzzyyykkknnnNWBWHWPPNjjjljkjkkhhhhhhhhhhhhfnfffhuyuyuyutjhjhjhggjjfgfhfhkhhjhjh},
\er{hjgghhghhffhkkmkljjnnnnghgghghjhyjyyjhhjhgjgjjhjhjhjhggjghjhnhmnmmbjoujhkhhh},
\er{hjgjgjgkllujkjjjkhhjhljjkjkljjljjhhhjkhhugg},
\er{hjgjgjgkllujkjjjkhhjhljjkjkljjljjhkhjhhffgjjjljjlmb} and in the
case of \underline{nearly rigid} elastic body also
\er{MaxVacFull1ninshtrgravortghhghgjkgghklhjgkghghjjkjhjkkggjkhjkhjjhhfhjhkjkhbbgjhzzzyyykkknnnNWBWHWPPNjjjljkjkkhhhhhhhhhhhhfnfffhuyuyuyutjhjhjhggjjfgfhfhjhhgjgj}
\er{hjgjgjgkllujkjjjkhhjhljjkjkljjljjhhhjkhhugghjjhjhgj} and
\er{hjgjgjgkllujkjjjkhhjhljjkjkljjljjhkhjhhffgjjjljjlmbhkhkhkhk} are
still valid if
\er{MaxVacFull1yiyiyninshtrgravortghhghgjkgghklhjgkghghjjkjhjkkggjkhjkhjjhhfhjhkjkhbbgjhzzzyyykkknnnNWBWHWPPNjjhjhjhjhplllhjggjl;kcc}
and
\er{MaxVacFull1yiyiyninshtrgravortghhghgjkgghklhjghjhnkkghghjjkjhjkkggjkhjkhjjhhfhjhkjkhbbgjhzzzyyykkknnnNWBWHWPPNjjhjhjhjhplllhjggjl;kcc}
hold only in some specific (possibly artificial) inertial or
non-inertial cartesian coordinate system.

Next, taking the divergence of both sides of
\er{MaxVacFull1ninshtrgravortghhghgjkgghklhjgkghghjjkjhjkkggjkhjkhjjhhfhjhkjkhbbgjhzzzyyykkknnnNWBWHWPPNjjjljkjkkhhhhhhhhhhhhfnfffhuyuyuyutjhjhjhggjjfgfhfhjhhgjgj}
and using
\er{MaxVacFull1yiyiyninshtrgravortghhghgjkgghklhjghjhnkkghghjjkjhjkkggjkhjkhjjhhfhjhkjkhbbgjhzzzyyykkknnnNWBWHWPPNjjhjhjhjhplllhjggjl;kcc}
gives:
\begin{multline}\label{MaxVacFull1ninshtrgravortghhghgjkgghklhjgkghghjjkjhjkkggjkhjkhjjhhfhjhkjkhbbgjhzzzyyykkknnnNWBWHWPPNjjjljkjkkhhhhhhhhhhhhfnfffhuyuyuyutjhjhjhggjjfgfhfhjhhgjgjhkjhjh}
\frac{\partial}{\partial
t}\left\{\mu_0\left(\frac{\partial}{\partial t}\left(\Div_{\vec
x}\vec g_1\right)+\vec u_0\cdot\nabla_{\vec x}\left(\Div_{\vec
x}\vec g_1\right)\right)\right\}+\Div_{\vec
x}\left\{\mu_0\left(\frac{\partial}{\partial t}\left(\Div_{\vec
x}\vec g_1\right)+\vec u_0\cdot\nabla_{\vec x}\left(\Div_{\vec
x}\vec
g_1\right)\right)\vec u_0\right\}\\
\approx\left(\alpha+\beta\right)\Delta_{\vec x}\left(\Div_{\vec
x}\vec g_1\right),
\end{multline}
i.e.:
\begin{equation}\label{MaxVacFull1ninshtrgravortghhghgjkgghklhjgkghghjjkjhjkkggjkhjkhjjhhfhjhkjkhbbgjhzzzyyykkknnnNWBWHWPPNjjjljkjkkhhhhhhhhhhhhfnfffhuyuyuyutjhjhjhggjjfgfhfhjhhgjgjhkjhjhkhkh}
\frac{\partial}{\partial
t}\left\{\frac{1}{c^2_{01}}\left(\frac{\partial h_1}{\partial
t}+\vec u_0\cdot\nabla_{\vec x} h_1\right)\right\}+\Div_{\vec
x}\left\{\frac{1}{c^2_{01}}\left(\frac{\partial h_1}{\partial
t}+\vec u_0\cdot\nabla_{\vec x} h_1\right)\vec u_0\right\}
\approx\Delta_{\vec x} h_1,
\end{equation}
where
\begin{equation}\label{hjgjgjgkllujkjjjkhhjhljjkjkljjljjhhhjkhhugghjjhjhgjkhhhkh}
h_1(\vec x,t):=\Div_{\vec x}\vec g_1(\vec
x,t)\quad\quad\text{and}\quad\quad
c_{01}:=\sqrt{\frac{\left(\alpha+\beta\right)}{\mu_0}}.
\end{equation}
On the other hand, taking the curl of both sides of
\er{MaxVacFull1ninshtrgravortghhghgjkgghklhjgkghghjjkjhjkkggjkhjkhjjhhfhjhkjkhbbgjhzzzyyykkknnnNWBWHWPPNjjjljkjkkhhhhhhhhhhhhfnfffhuyuyuyutjhjhjhggjjfgfhfhjh}
and using
\er{MaxVacFull1yiyiyninshtrgravortghhghgjkgghklhjghjhnkkghghjjkjhjkkggjkhjkhjjhhfhjhkjkhbbgjhzzzyyykkknnnNWBWHWPPNjjhjhjhjhplllhjggjl;kcc}
gives:
\begin{multline}\label{MaxVacFull1ninshtrgravortghhghgjkgghklhjgkghghjjkjhjkkggjkhjkhjjhhfhjhkjkhbbgjhzzzyyykkknnnNWBWHWPPNjjjljkjkkhhhhhhhhhhhhfnfffhuyuyuyutjhjhjhggjjfgfhfhjhgjg}
\frac{\partial}{\partial
t}\left\{\mu_0\left(\frac{\partial}{\partial t}\left(curl_{\vec
x}\,\vec g_1\right)+d_{\vec x}\left(curl_{\vec x}\,\vec
g_1\right)\cdot\vec u_0\right)\right\}+\Div_{\vec
x}\left\{\mu_0\left(\frac{\partial}{\partial t}\left(curl_{\vec
x}\,\vec g_1\right)+d_{\vec x}\left(curl_{\vec x}\,\vec
g_1\right)\cdot\vec u_0\right)\otimes\vec u_0\right\}\\
\approx\frac{\alpha}{2}\,\Delta_{\vec x}\left(curl_{\vec x}\,\vec
g_1\right),
\end{multline}
i.e.
\begin{equation}\label{MaxVacFull1ninshtrgravortghhghgjkgghklhjgkghghjjkjhjkkggjkhjkhjjhhfhjhkjkhbbgjhzzzyyykkknnnNWBWHWPPNjjjljkjkkhhhhhhhhhhhhfnfffhuyuyuyutjhjhjhggjjfgfhfhjhgjgjkhkhgnnjkh}
\frac{\partial}{\partial
t}\left\{\frac{1}{c^2_{02}}\left(\frac{\partial\vec h_2}{\partial
t}+d_{\vec x}\vec h_2\cdot\vec u_0\right)\right\}+\Div_{\vec
x}\left\{\frac{1}{c^2_{02}}\left(\frac{\partial\vec h_2}{\partial
t}+d_{\vec x}\vec h_2\cdot\vec u_0\right)\otimes\vec
u_0\right\}\approx\Delta_{\vec x}\vec h_2,
\end{equation}
where
\begin{equation}\label{hjgjgjgkllujkjjjkhhjhljjkjkljjljjhhhjkhhugghjjhjhgjkhhhkhhjjghjkhjh}
\vec h_2(\vec x,t):=curl_{\vec x}\,\vec g_1(\vec
x,t)\quad\quad\text{and}\quad\quad
c_{02}:=\sqrt{\frac{\alpha}{2\mu_0}}.
\end{equation}
Again, using
\er{MaxVacFull1yiyiyninshtrgravortghhghgjkgghklhjghjhnkkghghjjkjhjkkggjkhjkhjjhhfhjhkjkhbbgjhzzzyyykkknnnNWBWHWPPNjjhjhjhjhplllhjggjl;kcc}
we can rewrite
\er{MaxVacFull1ninshtrgravortghhghgjkgghklhjgkghghjjkjhjkkggjkhjkhjjhhfhjhkjkhbbgjhzzzyyykkknnnNWBWHWPPNjjjljkjkkhhhhhhhhhhhhfnfffhuyuyuyutjhjhjhggjjfgfhfhjhgjgjkhkhgnnjkh}
as:
\begin{multline}\label{MaxVacFullPPNnnnffffffyuughjhjhjhhjjkjhkkjhujgGGggff}
\frac{1}{c^2_{02}}\frac{\partial}{\partial
t}\left(\frac{\partial\vec h_2}{\partial t}-\vec { u}_0\times
curl_{\vec x}\vec h_2+\nabla_{\vec x}\left(\vec h_2\cdot\vec {
u}_0\right)\right)\\-\frac{1}{c^2_{02}}curl_{\vec x} \left\{\vec
{u}_0\times
\left(\frac{\partial\vec h_2}{\partial t}-\vec {u}_0\times
curl_{\vec x}\vec h_2+\nabla_{\vec x}\left(\vec h_2\cdot\vec {
u}_0\right)\right)\right\}\\+\frac{1}{c^2_{02}}\left(\Div_{\vec
x}\left(\frac{\partial\vec h_2}{\partial t}-\vec {u}_0\times
curl_{\vec x}\vec h_2+\nabla_{\vec x}\left(\vec h_2\cdot\vec {
u}_0\right)\right)\right)\vec {u}_0\approx \Delta_{\vec x}\vec h_2.
\end{multline}
Again note that, using Proposition \ref{yghgjtgyrtrt} we deduce that
equations
\er{MaxVacFull1ninshtrgravortghhghgjkgghklhjgkghghjjkjhjkkggjkhjkhjjhhfhjhkjkhbbgjhzzzyyykkknnnNWBWHWPPNjjjljkjkkhhhhhhhhhhhhfnfffhuyuyuyutjhjhjhggjjfgfhfhjhhgjgjhkjhjhkhkh}
and \er{MaxVacFullPPNnnnffffffyuughjhjhjhhjjkjhkkjhujgGGggff} are
invariant under the change of inertial or non-inertial cartesian
coordinate system, provided that $\vec h_2$ is a proper vector field
and $h_1$ is a proper scalar field. Thus, in the case of
\underline{nearly rigid} elastic body
\er{MaxVacFull1ninshtrgravortghhghgjkgghklhjgkghghjjkjhjkkggjkhjkhjjhhfhjhkjkhbbgjhzzzyyykkknnnNWBWHWPPNjjjljkjkkhhhhhhhhhhhhfnfffhuyuyuyutjhjhjhggjjfgfhfhjhhgjgjhkjhjhkhkh}
and \er{MaxVacFullPPNnnnffffffyuughjhjhjhhjjkjhkkjhujgGGggff} are
still valid if
\er{MaxVacFull1yiyiyninshtrgravortghhghgjkgghklhjgkghghjjkjhjkkggjkhjkhjjhhfhjhkjkhbbgjhzzzyyykkknnnNWBWHWPPNjjhjhjhjhplllhjggjl;kcc}
and
\er{MaxVacFull1yiyiyninshtrgravortghhghgjkgghklhjghjhnkkghghjjkjhjkkggjkhjkhjjhhfhjhkjkhbbgjhzzzyyykkknnnNWBWHWPPNjjhjhjhjhplllhjggjl;kcc}
hold only in some specific (possibly artificial) inertial or
non-inertial cartesian coordinate system. Finally note that,
although the equation
\er{MaxVacFull1ninshtrgravortghhghgjkgghklhjgkghghjjkjhjkkggjkhjkhjjhhfhjhkjkhbbgjhzzzyyykkknnnNWBWHWPPNjjjljkjkkhhhhhhhhhhhhfnfffhuyuyuyutjhjhjhggjjfgfhfhjhgjgjkhkhgnnjkh}
is invariant under the Galilean Transformation, it is not invariant
under the more general change of non-inertial cartesian coordinate
system. However,
\er{MaxVacFull1ninshtrgravortghhghgjkgghklhjgkghghjjkjhjkkggjkhjkhjjhhfhjhkjkhbbgjhzzzyyykkknnnNWBWHWPPNjjjljkjkkhhhhhhhhhhhhfnfffhuyuyuyutjhjhjhggjjfgfhfhjhgjgjkhkhgnnjkh}
is more convenient than
\er{MaxVacFullPPNnnnffffffyuughjhjhjhhjjkjhkkjhujgGGggff}, since
every of the three scalar components of the vector field $\vec h_2$
in
\er{MaxVacFull1ninshtrgravortghhghgjkgghklhjgkghghjjkjhjkkggjkhjkhjjhhfhjhkjkhbbgjhzzzyyykkknnnNWBWHWPPNjjjljkjkkhhhhhhhhhhhhfnfffhuyuyuyutjhjhjhggjjfgfhfhjhgjgjkhkhgnnjkh}
satisfies three decoupled wave equations of the same type. On the
other hand, if we consider some three proper vector fields $\vec
e_1:=\vec e_1(\vec x,t)$, $\vec e_2:=\vec e_2(\vec x,t)$, and $\vec
e_3:=\vec e_3(\vec x,t)$, which are mutually orthogonal to each
other and satisfy the following approximation analogous to
\er{MaxVacFull1yiyiyninshtrgravortghhghgjkgghklhjghjhnkkghghjjkjhjkkggjkhjkhjjhhfhjhkjkhbbgjhzzzyyykkknnnNWBWHWPPNjjhjhjhjhplllhjggjl;kcc}:
\begin{equation}\label{MaxVacFullPPNmmmffffffhhtygghGGGGyuhggghghhjhjhff}
\frac{|d_{\vec x}\vec {e}_1|+c_{02}|\partial_t\vec e_1|}{|\vec
{e}_1|}+\frac{|d_{\vec x}\vec {e}_2|+c_{02}|\partial_t\vec
e_2|}{|\vec {e}_2|}+\frac{|d_{\vec x}\vec
{e}_3|+c_{02}|\partial_t\vec e_3|}{|\vec {e}_3|}\ll\frac{|d_{\vec
x}\vec h_2|+c_{02}|\partial_t\vec h_2|}{|\vec h_2|}.
\end{equation}
i.e. the field $\vec {e}_k$ vary in space and time much weaker than
$\vec h_2$, then we may write the alternative to
\er{MaxVacFull1ninshtrgravortghhghgjkgghklhjgkghghjjkjhjkkggjkhjkhjjhhfhjhkjkhbbgjhzzzyyykkknnnNWBWHWPPNjjjljkjkkhhhhhhhhhhhhfnfffhuyuyuyutjhjhjhggjjfgfhfhjhgjgjkhkhgnnjkh}
approximate equations in the form of three decoupled scalar wave
equations of the same type:
\begin{multline}\label{MaxVacFullPPNmmmffffffiuiuhjuGGggFGjjkkjjff}
\frac{\partial}{\partial
t}\left\{\frac{1}{c^2_{02}}\left(\frac{\partial}{\partial
t}\left(\vec e_k\cdot\vec h_2\right)+\vec {u}_0\cdot\nabla_{\vec
x}\left(\vec e_k\cdot\vec h_2\right)\right)\right\}+\Div_{\vec x}
\left\{\frac{1}{c^2_{02}}\left(\frac{\partial}{\partial t}\left(\vec
e_k\cdot\vec h_2\right)+\vec {u}_0\cdot\nabla_{\vec x}\left(\vec
e_k\cdot\vec h_2\right)\right)\vec {u}_0\right\}\\ \approx
\Delta_{\vec x}\left(\vec e_k\cdot\vec h_2\right)\quad\quad\forall
k=1,2,3.
\end{multline}
Then, clearly, the new alternative approximate equations
\er{MaxVacFullPPNmmmffffffiuiuhjuGGggFGjjkkjjff} together with
\er{MaxVacFull1ninshtrgravortghhghgjkgghklhjgkghghjjkjhjkkggjkhjkhjjhhfhjhkjkhbbgjhzzzyyykkknnnNWBWHWPPNjjjljkjkkhhhhhhhhhhhhfnfffhuyuyuyutjhjhjhggjjfgfhfhjhhgjgjhkjhjhkhkh}
are indeed invariant under the more general change of non-inertial
cartesian coordinate system.

As a final corollary of the above in the case of nearly rigid body
we have \er{MaxVacFullPPNmmmffffffiuiuhjuGGggFGjjkkjjff} and
\er{MaxVacFull1ninshtrgravortghhghgjkgghklhjgkghghjjkjhjkkggjkhjkhjjhhfhjhkjkhbbgjhzzzyyykkknnnNWBWHWPPNjjjljkjkkhhhhhhhhhhhhfnfffhuyuyuyutjhjhjhggjjfgfhfhjhhgjgjhkjhjhkhkh},
i.e.:
\begin{equation}\label{MaxVacFull1ninshtrgravortghhghgjkgghklhjgkghghjjkjhjkkggjkhjkhjjhhfhjhkjkhbbgjhzzzyyykkknnnNWBWHWPPNjjjljkjkkhhhhhhhhhhhhfnfffhuyuyuyutjhjhjhggjjfgfhfhjhhgjgjhkjhjhkhkhjj}
\frac{\partial}{\partial
t}\left\{\frac{1}{c^2_{01}}\left(\frac{\partial h_1}{\partial
t}+\vec u_0\cdot\nabla_{\vec x} h_1\right)\right\}+\Div_{\vec
x}\left\{\frac{1}{c^2_{01}}\left(\frac{\partial h_1}{\partial
t}+\vec u_0\cdot\nabla_{\vec x} h_1\right)\vec u_0\right\}
\approx\Delta_{\vec x} h_1,
\end{equation}
and
\begin{multline}\label{MaxVacFullPPNmmmffffffiuiuhjuGGggFGjjkkjjffjj}
\frac{\partial}{\partial
t}\left\{\frac{1}{c^2_{02}}\left(\frac{\partial}{\partial
t}\left(\vec e_k\cdot\vec h_2\right)+\vec {u}_0\cdot\nabla_{\vec
x}\left(\vec e_k\cdot\vec h_2\right)\right)\right\}+\Div_{\vec x}
\left\{\frac{1}{c^2_{02}}\left(\frac{\partial}{\partial t}\left(\vec
e_k\cdot\vec h_2\right)+\vec {u}_0\cdot\nabla_{\vec x}\left(\vec
e_k\cdot\vec h_2\right)\right)\vec {u}_0\right\}\\ \approx
\Delta_{\vec x}\left(\vec e_k\cdot\vec h_2\right)\quad\quad\forall
k=1,2,3,
\end{multline}
where
\begin{equation}\label{hjgjgjgkllujkjjjkhhjhljjkjkljjljjhhhjkhhugghjjhjhgjkhhhkhjkhhkhkkh}
h_1(\vec x,t):=\Div_{\vec x}\vec g_1(\vec
x,t)\quad\quad\text{and}\quad\quad
c_{01}:=\sqrt{\frac{\left(\alpha+\beta\right)}{\mu_0}}.
\end{equation}
and
\begin{equation}\label{hjgjgjgkllujkjjjkhhjhljjkjkljjljjhhhjkhhugghjjhjhgjkhhhkhhjjghjkhjhjkjkjkk}
\vec h_2(\vec x,t):=curl_{\vec x}\,\vec g_1(\vec
x,t)\quad\quad\text{and}\quad\quad
c_{02}:=\sqrt{\frac{\alpha}{2\mu_0}}.
\end{equation}
Moreover, we have
\er{hjgjgjgkllujkjjjkhhjhljjkjkljjljjhkhjhhffgjjjljjlmbhkhkhkhk},
i.e.:
\begin{equation}\label{hjgjgjgkllujkjjjkhhjhljjkjkljjljjhkhjhhffgjjjljjlmbhkhkhkhkjkkhhh}
\mathcal{E}_{1}\approx-\frac{1}{2}\left(d_{\vec x}\vec g_1+\left\{
d_{\vec x}\vec g_1\right\}^T\right),
\end{equation}
and
\er{hjgghhghhffhkkmkljjnnnnghgghghjhyjyyjhhjhgjgjjhjhjhjhggjghjhnhmnmmbjoujhkhhh},
i.e.:
\begin{equation}\label{hjgghhghhffhkkmkljjnnnnghgghghjhyjyyjhhjhgjgjjhjhjhjhggjghjhnhmnmmbjoujhkggghgh}
\frac{\partial\vec g_1}{\partial t}-curl_{\vec x}\left\{\vec
u_0\times\vec g_1\right\}+\left(\Div_{\vec x}\vec g_1\right)\vec
u_0\approx-\vec u_1,
\end{equation}
and the above six equations are invariant under the change of
inertial or non-inertial cartesian coordinate system.

As we can easily see, the divergence of $\vec g_1$ and any of the
scalar components of the curl of $\vec g_1$ satisfy the invariant
wave equation of the same type as
\er{MaxVacFull1ninshtrgravortghhghgjkgghklhjgkghghjjkjhjkkggjkhjkhjjhhfhjhkjhjjhkhbbgjlmkmmhzzzyyykkknnnNWBWHWPPNjjlllkkkkkkljkkjjkjkhjjhhgghgggghjjhgjggh}
or
\er{MaxVacFull1ninshtrgravortghhghgjkgghklhjgkghghjjkjhjkkggjkhjkhjjhhfhjhkjhjjhkhbbgjlmkmmhzzzyyykkknnnNWBWHWPPNjjlllkkkkkkljkkjjkjkhjjhhgghgggghjjhgjgghjkjkjk}.
However, as we can see from
\er{hjgjgjgkllujkjjjkhhjhljjkjkljjljjhhhjkhhugghjjhjhgjkhhhkhjkhhkhkkh}
and
\er{hjgjgjgkllujkjjjkhhjhljjkjkljjljjhhhjkhhugghjjhjhgjkhhhkhhjjghjkhjhjkjkjkk}
the characteristic parameter $c_{01}$ in the wave equation for the
divergence part of $\vec g_1$
\er{MaxVacFull1ninshtrgravortghhghgjkgghklhjgkghghjjkjhjkkggjkhjkhjjhhfhjhkjkhbbgjhzzzyyykkknnnNWBWHWPPNjjjljkjkkhhhhhhhhhhhhfnfffhuyuyuyutjhjhjhggjjfgfhfhjhhgjgjhkjhjhkhkhjj}
differ from the characteristic parameter $c_{02}$ in the wave
equation for the curl part of $\vec g_1$
\er{MaxVacFullPPNmmmffffffiuiuhjuGGggFGjjkkjjffjj}, i.e. the
divergence part and the curl parts of  $\vec g_1$ propagate as two
different waves with the different speeds. As it can be easily seen,
in the case of the flat waves in the resting body the divergence
part of $\vec g_1$ (which is curl-free) propagate as a
\underline{longitudinal} wave, similarly to the sound in fluid or
gas, and at the same time the curl part of $\vec g_1$ (which is
divergence-free) propagate as a \underline{transverse} wave.
Moreover, the longitudinal and transverse waves in an elastic body
propagate with two different speeds.

\section{Maxwell equations in the presence of Dielectrics and/or Magnetics}\label{DMPGG}

\subsection{Polarization and Magnetization of a totally neutral
system of point charges}\label{huugftdtdtddtdt}
\begin{proposition}\label{hjghfhfhf11}
Consider an arbitrary moving point with place $\vec r(t)$ and
velocity $\frac{d\vec r}{dt}(t)$ and let $\vec u(\vec x,t)$ be a
speed like vector field. Next consider a totally neutral system of
$n$ point charges $\sigma_1,\ldots,\sigma_n$, with places $\vec
r_1(t),\ldots,\vec r_n(t)$ and velocities $\frac{d\vec
r_1}{dt}(t),\ldots, \frac{d\vec r_n}{dt}(t)$ satisfying
\begin{equation}\label{hgkgjhfhjfjfiouykkhhhjjgjgkhkhkjuj11yyttt11}
\sum_{k=1}^{n}\sigma_k=0.
\end{equation}
Then, denoting, as usual, the charge and the current densities as:
\begin{equation}\label{hgkgjhfhjfjfiouykkhhhjhjhjhjhh11ghghghf11kljjkkkhhuuug11}
\rho(\vec x,t):=\sum_{k=1}^{n}\sigma_k\delta\left(\vec x-\vec
r_k(t)\right),\quad\text{and}\quad \vec j(\vec x,t)
:=\sum_{k=1}^{n}\sigma_k\frac{d\vec r_k}{dt}(t)\,\delta\left(\vec
x-\vec r_k(t)\right),
\end{equation}
we have
\begin{equation}\label{hgkgjhfhjfjfiouykkhhhjhjhjhjhhkjljljl;lk;lk;l;k;utityi11}
\begin{cases}
\rho(\vec x,t)=-div_{\vec x}\vec P(\vec x,t)\\
\vec j(x,t)= \frac{\partial\vec P}{\partial t}(\vec x,t)-curl_{\vec
x}\left\{\vec u\left(\vec x,t\right)\times\vec P(\vec
x,t)\right\}+c\,curl_{\vec x}\vec M_{\vec u}(\vec x,t),
\end{cases}
\end{equation}
or equivalently
\begin{equation}\label{hgkgjhfhjfjfiouykkhhhjhjhjhjhhkjljljl;lk;lk;l;k;utityi11jijij}
\begin{cases}
\rho(\vec x,t)=-div_{\vec x}\vec P(\vec x,t)\\
\vec j(x,t)= \frac{\partial\vec P}{\partial t}(\vec
x,t)+c\,curl_{\vec x}\vec M(\vec x,t),
\end{cases}
\end{equation}
where $\vec P(\vec x,t)$, $\vec M(\vec x,t)$ and $\vec M_{\vec
u}(\vec x,t)$ are the Polarization and the Magnetization of the
given neutral system defined as:
\begin{multline}\label{hgkgjhfhjfjfiouykkhhhjjgjgkhkhkjujuyfffffffffffffffhjjhjhj11}
\vec P(\vec x,t):=\sum_{k=1}^{n}\sigma_k\vec
l_k(t)\int_0^1\,\delta\left(\vec x-\left(\vec r(t)+s\vec
l_k(t)\right)\right)ds\,,\\
\vec M(\vec
x,t):=\\-\frac{1}{c}\left(\sum_{k=1}^{n}\sigma_k\int_0^1\left(\left(\frac{d\vec
r}{dt}(t)+s\frac{d\vec l_k}{dt}(t)\right)\times\vec
l_k(t)\,\delta\left(\vec
x-\left(\vec r(t)+s\vec l_k(t)\right)\right)\right)ds\right)\\
\quad\text{and}\quad\quad\quad\vec M_{\vec u}(\vec x,t):=\vec M(\vec
x,t)+\frac{1}{c}\vec u(\vec x,t)\times \vec P(\vec
x,t)=\\-\frac{1}{c}\left(\sum_{k=1}^{n}\sigma_k\int_0^1\left(\left(\left(\frac{d\vec
r}{dt}(t)+s\frac{d\vec l_k}{dt}(t)\right)-\vec u\left(\vec
r(t)+s\vec l_k(t),t\right)\right)\times\vec l_k(t)\,\delta\left(\vec
x-\left(\vec r(t)+s\vec l_k(t)\right)\right)\right)ds\right),
\end{multline}
where we denoted
\begin{equation}\label{hgkgjhfhjfjfiouykkhhhjhjhjhjhh11ghghghf11kljjkkkhhuuugnbhnhkhhn11}
\vec l_k(t):=\vec r_k(t)-\vec r(t)\quad\quad\forall\,
k\,=\,1,2,\ldots, n.
\end{equation}
Next, under the change of cartesian coordinate system $(*)$ to
another cartesian coordinate system $(**)$ of the form
\begin{equation*}
\begin{cases}
\vec x'=A(t)\cdot \vec x+\vec z(t),\\
t'=t,
\end{cases}
\end{equation*}
where $A(t)\in SO(3)$ is a rotation, we have
\begin{equation}\label{guigikvbvbGGjkilil11}
\begin{cases}
\vec P'(\vec x',t')=A(t)\cdot\vec P(\vec x,t),\\
\vec M'_{\vec u'}(\vec x',t')=A(t)\cdot\vec M_{\vec u}(\vec x,t),\\
\vec M'=A(t)\cdot\vec M-\frac{1}{c}\,\left(\frac{dA}{dt}(t)\cdot\vec
x+\frac{d\vec z}{dt}(t)\right)\times \left(A(t)\cdot\vec P\right),\\
\rho'(\vec x',t')=\rho(\vec x,t),
\\
\vec j'(\vec x',t')=A(t)\cdot \vec j(\vec x,t)+\rho(\vec
x,t)\,\frac{dA}{dt}(t)\cdot\vec x+\rho(\vec x,t)\,\frac{d\vec
z}{dt}(t),
\end{cases}
\end{equation}
provided that
\begin{equation}\label{guigikvbvbGGjkililyjjkuhjhgf11}
\vec u'(\vec x',t')=A(t)\cdot \vec u(\vec
x,t)+\frac{dA}{dt}(t)\cdot\vec x+\frac{d\vec z}{dt}(t).
\end{equation}
Moreover, identities in
\er{hgkgjhfhjfjfiouykkhhhjhjhjhjhhkjljljl;lk;lk;l;k;utityi11} are
invariant under the change of cartesian coordinate system systems,
provided we have \er{guigikvbvbGGjkilil11} and
\er{guigikvbvbGGjkililyjjkuhjhgf11}. Finally, if all the quantities
$\big|\vec l_k(t)\big|$ in
\er{hgkgjhfhjfjfiouykkhhhjhjhjhjhh11ghghghf11kljjkkkhhuuugnbhnhkhhn11}
are infinitely small, then we have the following approximate
equalities for $\vec P(\vec x,t)$, $\vec M(\vec x,t)$ and $\vec
M_{\vec u}(\vec x,t)$ in
\er{hgkgjhfhjfjfiouykkhhhjjgjgkhkhkjujuyfffffffffffffffhjjhjhj11}:
\begin{multline}\label{hgkgjhfhjfjfiouykkhhhjjgjgkhkhkjujuyfffffffffffffffhjjhjhji99u9u9j11}
\vec P(\vec x,t)\approx \sum_{k=1}^{n}\sigma_k\vec
l_k(t)\,\delta\left(\vec x-\vec r(t)\right)\,,\\
\vec M(\vec
x,t)\approx-\frac{1}{c}\left(\sum_{k=1}^{n}\sigma_k\frac{d\vec
r}{dt}(t)\times\vec l_k(t)\,\delta\left(\vec x-\vec
r(t)\right)\right)
\quad\quad\quad\text{and}\quad\\
\vec M_{\vec u}(\vec
x,t)\approx-\frac{1}{c}\left(\sum_{k=1}^{n}\sigma_k\left(\frac{d\vec
r}{dt}(t)-\vec u\left(\vec r(t),t\right)\right)\times\vec
l_k(t)\,\delta\left(\vec x-\vec r(t)\right)\right) \approx\vec
M(\vec x,t)+\frac{1}{c}\vec u(\vec x,t)\times \vec P(\vec x,t)
\end{multline}
and moreover, if $\vec P$, $\vec M$ and $\vec M_{\vec u}$ are given
by
\er{hgkgjhfhjfjfiouykkhhhjjgjgkhkhkjujuyfffffffffffffffhjjhjhji99u9u9j11}
rather than
\er{hgkgjhfhjfjfiouykkhhhjjgjgkhkhkjujuyfffffffffffffffhjjhjhj11},
then, under the change of cartesian coordinate system, the first
three equations in \er{guigikvbvbGGjkilil11} are still valid.
\end{proposition}
\begin{proof}
Let
\begin{multline}\label{hgkgjhfhjfjfiouykkhhhjjgjgkhkhkjujuyfffffffffffffffhjjhjhj11uiuyuyjii}
\vec P(\vec x,t):=\sum_{k=1}^{n}\sigma_k\vec
l_k(t)\int_0^1\,\delta\left(\vec x-\left(\vec r(t)+s\vec
l_k(t)\right)\right)ds\,,\\
\vec M(\vec
x,t):=\\-\frac{1}{c}\left(\sum_{k=1}^{n}\sigma_k\int_0^1\left(\left(\frac{d\vec
r}{dt}(t)+s\frac{d\vec l_k}{dt}(t)\right)\times\vec
l_k(t)\,\delta\left(\vec
x-\left(\vec r(t)+s\vec l_k(t)\right)\right)\right)ds\right)\\
\quad\text{and}\quad\quad\quad\vec M_{\vec u}(\vec
x,t):=\\-\frac{1}{c}\left(\sum_{k=1}^{n}\sigma_k\int_0^1\left(\left(\left(\frac{d\vec
r}{dt}(t)+s\frac{d\vec l_k}{dt}(t)\right)-\vec u\left(\vec
r(t)+s\vec l_k(t),t\right)\right)\times\vec l_k(t)\,\delta\left(\vec
x-\left(\vec r(t)+s\vec l_k(t)\right)\right)\right)ds\right).
\end{multline}
Then, by definition in
\er{hgkgjhfhjfjfiouykkhhhjjgjgkhkhkjujuyfffffffffffffffhjjhjhj11uiuyuyjii}
we obviously have:
\begin{equation}\label{hgkguiuioiioioioyyy}
\vec M_{\vec u}(\vec x,t)=\vec M(\vec x,t)+\frac{1}{c}\vec u(\vec
x,t)\times \vec P(\vec x,t)\,.
\end{equation}
Next, by
\er{hgkgjhfhjfjfiouykkhhhjhjhjhjhh11ghghghf11kljjkkkhhuuug11} and
\er{hgkgjhfhjfjfiouykkhhhjjgjgkhkhkjuj11yyttt11} we have
\begin{multline}\label{hgkgjhfhjfjfiouykkhhhj11}
\rho(\vec x,t):=\sum_{k=1}^{n}\sigma_k\delta\left(\vec x-\left(\vec
r(t)+\vec
l_k(t)\right)\right)=\sum_{k=1}^{n}\sigma_k\left(\delta\left(\vec
x-\left(\vec r(t)+\vec l_k(t)\right)\right)-\delta\left(\vec x-\vec
r(t)\right)\right)=\\-\sum_{k=1}^{n}\sigma_k\int_0^1\vec
l_k(t)\cdot\nabla_{\vec x}\delta\left(\vec x-\left(\vec r(t)+s\vec
l_k(t)\right)\right)ds=\\-div_{\vec
x}\left(\sum_{k=1}^{n}\sigma_k\vec l_k(t)\int_0^1\,\delta\left(\vec
x-\left(\vec r(t)+s\vec l_k(t)\right)\right)ds\right).
\end{multline}
Thus
\begin{equation}\label{hgkgjhfhjfjfiouykkhhhjhjhjhjhh11ghghghf11}
\rho(\vec x,t)=-div_{\vec x}\vec P(\vec x,t),
\end{equation}
where
\begin{equation}\label{hgkgjhfhjfjfiouykkhhhjhjhjhjhh11ghghghf11gffjff11}
\vec P(\vec x,t)=\sum_{k=1}^{n}\sigma_k\vec
l_k(t)\int_0^1\,\delta\left(\vec x-\left(\vec r(t)+s\vec
l_k(t)\right)\right)ds.
\end{equation}
Moreover, if all the quantities $\big|\vec l_k(t)\big|$ are
infinitely small, then by
\er{hgkgjhfhjfjfiouykkhhhjhjhjhjhh11ghghghf11gffjff11} we obviously
have the following approximate equality
\begin{equation}\label{hgkgjhfhjfjfiouykkhhhjhjhjhjhh11ghghghf11gffjff11kjhhjhj11}
\vec P(\vec x,t)\approx \sum_{k=1}^{n}\sigma_k\vec
l_k(t)\,\delta\left(\vec x-\vec r(t)\right).
\end{equation}
Next by
\er{hgkgjhfhjfjfiouykkhhhjhjhjhjhh11ghghghf11kljjkkkhhuuug11} and
\er{hgkgjhfhjfjfiouykkhhhjjgjgkhkhkjuj11yyttt11} we obtain
\begin{multline}\label{hgkgjhfhjfjfiouykkhhhjjgjg11}
\vec j(\vec x,t) :=\sum_{k=1}^{n}\sigma_k\left(\frac{d\vec
r}{dt}(t)+\frac{d\vec l_k}{dt}(t)\right)\delta\left(\vec
x-\left(\vec r(t)+\vec l_k(t)\right)\right)=\\
\sum_{k=1}^{n}\sigma_k\left(\left(\frac{d\vec r}{dt}(t)+\frac{d\vec
l_k}{dt}(t)\right)\delta\left(\vec x-\left(\vec r(t)+\vec
l_k(t)\right)\right)-\frac{d\vec r}{dt}(t)\delta\left(\vec x-\vec
r(t)\right)\right)=\\
\sum_{k=1}^{n}\sigma_k\int_0^1\frac{\partial}{\partial
s}\left(\left(\frac{d\vec r}{dt}(t)+s\frac{d\vec
l_k}{dt}(t)\right)\delta\left(\vec x-\left(\vec r(t)+s\vec
l_k(t)\right)\right)\right)ds
=\\
\sum_{k=1}^{n}\sigma_k\int_0^1\left(\frac{d\vec
l_k}{dt}(t)\delta\left(\vec x-\left(\vec r(t)+s\vec
l_k(t)\right)\right)- \left(\frac{d\vec r}{dt}(t)+s\frac{d\vec
l_k}{dt}(t)\right)\left(\vec l_k(t)\cdot\nabla_{\vec
x}\delta\left(\vec x-\left(\vec r(t)+s\vec
l_k(t)\right)\right)\right)\right)ds
=\\
\sum_{k=1}^{n}\sigma_k\int_0^1\left(\frac{d\vec
l_k}{dt}(t)\delta\left(\vec x-\left(\vec r(t)+s\vec
l_k(t)\right)\right)- \vec l_k(t)\left(\left(\frac{d\vec
r}{dt}(t)+s\frac{d\vec l_k}{dt}(t)\right)\cdot\nabla_{\vec
x}\delta\left(\vec x-\left(\vec r(t)+s\vec
l_k(t)\right)\right)\right)\right)ds
\\+\sum_{k=1}^{n}\sigma_k\int_0^1\left(\vec
l_k(t)\left(\left(\frac{d\vec r}{dt}(t)+s\frac{d\vec
l_k}{dt}(t)\right)\cdot\nabla_{\vec x}\delta\left(\vec x-\left(\vec
r(t)+s\vec l_k(t)\right)\right)\right)\right)ds
\\
-\sum_{k=1}^{n}\sigma_k\int_0^1\left(\left(\frac{d\vec
r}{dt}(t)+s\frac{d\vec l_k}{dt}(t)\right)\left(\vec
l_k(t)\cdot\nabla_{\vec x}\delta\left(\vec x-\left(\vec r(t)+s\vec
l_k(t)\right)\right)\right)\right)ds
=\\
\frac{\partial}{\partial t}\left(\sum_{k=1}^{n}\sigma_k\vec
l_k(t)\int_0^1\delta\left(\vec x-\left(\vec r(t)+s\vec
l_k(t)\right)\right)ds\right)
\\+div_{\vec
x}\left\{\sum_{k=1}^{n}\sigma_k\int_0^1\left(\vec
l_k(t)\otimes\left(\frac{d\vec r}{dt}(t)+s\frac{d\vec
l_k}{dt}(t)\right)\delta\left(\vec x-\left(\vec r(t)+s\vec
l_k(t)\right)\right)\right)ds\right\}
\\
-div_{\vec
x}\left\{\sum_{k=1}^{n}\sigma_k\int_0^1\left(\left(\frac{d\vec
r}{dt}(t)+s\frac{d\vec l_k}{dt}(t)\right)\otimes\vec
l_k(t)\,\delta\left(\vec x-\left(\vec r(t)+s\vec
l_k(t)\right)\right)\right)ds\right\}.
\end{multline}
Then, using \er{apfrm6kkk}, by \er{hgkgjhfhjfjfiouykkhhhjjgjg11} we
deduce
\begin{multline}\label{hgkgjhfhjfjfiouykkhhhjjgjgkhkhk1111}
\vec j(\vec x,t)= \frac{\partial}{\partial
t}\left(\sum_{k=1}^{n}\sigma_k\vec l_k(t)\int_0^1\delta\left(\vec
x-\left(\vec r(t)+s\vec l_k(t)\right)\right)ds\right)
\\
-curl_{\vec
x}\left\{\sum_{k=1}^{n}\sigma_k\int_0^1\left(\left(\frac{d\vec
r}{dt}(t)+s\frac{d\vec l_k}{dt}(t)\right)\times\vec
l_k(t)\,\delta\left(\vec x-\left(\vec r(t)+s\vec
l_k(t)\right)\right)\right)ds\right\}.
\end{multline}
If $\vec u(\vec x,t)$ is the speed like vector field, then by
\er{hgkgjhfhjfjfiouykkhhhjjgjgkhkhk1111} we have
\begin{multline}\label{hgkgjhfhjfjfiouykkhhhjjgjgkhkhk11}
\vec j(\vec x,t)= \frac{\partial}{\partial
t}\left(\sum_{k=1}^{n}\sigma_k\vec l_k(t)\int_0^1\delta\left(\vec
x-\left(\vec r(t)+s\vec l_k(t)\right)\right)ds\right)
\\
-curl_{\vec x}\left\{\sum_{k=1}^{n}\sigma_k\int_0^1\left(\vec
u\left(\vec r(t)+s\vec l_k(t),t\right)\times\vec
l_k(t)\,\delta\left(\vec x-\left(\vec r(t)+s\vec
l_k(t)\right)\right)\right)ds\right\}\\
-curl_{\vec
x}\left\{\sum_{k=1}^{n}\sigma_k\int_0^1\left(\left(\left(\frac{d\vec
r}{dt}(t)+s\frac{d\vec l_k}{dt}(t)\right)-\vec u\left(\vec
r(t)+s\vec l_k(t),t\right)\right)\times\vec l_k(t)\,\delta\left(\vec
x-\left(\vec r(t)+s\vec l_k(t)\right)\right)\right)ds\right\}\\
= \frac{\partial}{\partial t}\left(\sum_{k=1}^{n}\sigma_k\vec
l_k(t)\int_0^1\delta\left(\vec x-\left(\vec r(t)+s\vec
l_k(t)\right)\right)ds\right)
\\
-curl_{\vec x}\left\{\vec u\left(\vec
x,t\right)\times\left(\sum_{k=1}^{n}\sigma_k\vec
l_k(t)\int_0^1\delta\left(\vec x-\left(\vec r(t)+s\vec
l_k(t)\right)\right)ds\right)\right\}\\
-curl_{\vec
x}\left\{\sum_{k=1}^{n}\sigma_k\int_0^1\left(\left(\left(\frac{d\vec
r}{dt}(t)+s\frac{d\vec l_k}{dt}(t)\right)-\vec u\left(\vec
r(t)+s\vec l_k(t),t\right)\right)\times\vec l_k(t)\,\delta\left(\vec
x-\left(\vec r(t)+s\vec l_k(t)\right)\right)\right)ds\right\}.
\end{multline}
Thus, denoting by \er{hgkgjhfhjfjfiouykkhhhj11} and
\er{hgkgjhfhjfjfiouykkhhhjjgjgkhkhkjujuyfffffffffffffffhjjhjhj11uiuyuyjii}
we have
\begin{equation}\label{hgkgjhfhjfjfiouykkhhhjhjhjhjhh11}
\rho(\vec x,t)=-div_{\vec x}\vec P(\vec x,t),
\end{equation}
and by \er{hgkgjhfhjfjfiouykkhhhjjgjgkhkhk11} and
\er{hgkgjhfhjfjfiouykkhhhjjgjgkhkhkjujuyfffffffffffffffhjjhjhj11uiuyuyjii}
we have
\begin{equation}\label{hgkgjhfhjfjfiouykkhhhjjgjgkhkhkjhgutfhfhfyiyy11}
\vec j(x,t)= \frac{\partial\vec P}{\partial t}(\vec x,t)-curl_{\vec
x}\left\{\vec u\left(\vec x,t\right)\times\vec P(\vec
x,t)\right\}+c\,curl_{\vec x}\vec M_{\vec u}(\vec x,t).
\end{equation}
Furthermore, by \er{hgkgjhfhjfjfiouykkhhhjjgjgkhkhkjhgutfhfhfyiyy11}
and \er{hgkguiuioiioioioyyy} we have
\begin{equation}\label{hgkgjhfhjfjfiouykkhhhjjgjgkhkhkjhgutfhfhfyiyy11hjhh}
\vec j(x,t)= \frac{\partial\vec P}{\partial t}(\vec
x,t)+c\,curl_{\vec x}\vec M(\vec x,t).
\end{equation}
Moreover, if all the quantities $\big|\vec l_k(t)\big|$ are
infinitely small, then by
\er{hgkgjhfhjfjfiouykkhhhjjgjgkhkhkjujuyfffffffffffffffhjjhjhj11uiuyuyjii}
we obviously have the following approximate equalities
\begin{multline}\label{hgkgjhfhjfjfiouykkhhhjjgjgkhkhkjujkjjkjkjk11}
\vec P(\vec x,t)\approx \sum_{k=1}^{n}\sigma_k\vec
l_k(t)\,\delta\left(\vec x-\vec r(t)\right)\,,\\
\vec M(\vec x,t) \approx
-\frac{1}{c}\left(\sum_{k=1}^{n}\sigma_k\frac{d\vec
r}{dt}(t)\times\vec l_k(t)\,\delta\left(\vec x-\vec
r(t)\right)\right) \quad\text{and}\\
\quad \vec M_{\vec u}(\vec x,t)
\approx -\frac{1}{c}\left(\sum_{k=1}^{n}\sigma_k\left(\frac{d\vec
r}{dt}(t)-\vec u\left(\vec r(t),t\right)\right)\times\vec
l_k(t)\,\delta\left(\vec x-\vec r(t)\right)\right)\approx\vec M(\vec
x,t)+\frac{1}{c}\vec u(\vec x,t)\times \vec P(\vec x,t) .
\end{multline}

Next, consider the change of cartesian coordinate system system
$(*)$ to another cartesian coordinate system $(**)$ of the form
\begin{equation*}
\begin{cases}
\vec x'=A(t)\cdot \vec x+\vec z(t),\\
t'=t,
\end{cases}
\end{equation*}
where $A(t)\in SO(3)$ is a rotation. Then obviously we have:
\begin{equation}\label{guigikvbvbGGjkilil11uhggy11}
\begin{cases}
\vec u'(\vec x',t')=A(t)\cdot \vec u(\vec
x,t)+\frac{dA}{dt}(t)\cdot\vec x+\frac{d\vec z}{dt}(t)\\
\rho'(\vec x',t')=\rho(\vec x,t),
\\
\vec j'(\vec x',t')=A(t)\cdot \vec j(\vec x,t)+\rho(\vec
x,t)\,\frac{dA}{dt}(t)\cdot\vec x+\rho(\vec x,t)\,\frac{d\vec
z}{dt}(t),
\\
\vec r'(t')=A(t)\cdot\vec r(t)+\vec z(t),
\\
\frac{d\vec r'}{dt'}(t')=A(t)\cdot\frac{d\vec
r}{dt}(t)+\frac{dA}{dt}(t)\cdot\vec r(t)+\frac{d\vec z}{dt}(t),
\\
\vec r'_k(t')=A(t)\cdot\vec r_k(t)+\vec z(t)\quad\quad\forall\,
k\,=\,1,2,\ldots, n,
\\
\frac{d\vec r'_k}{dt'}(t')=A(t)\cdot\frac{d\vec
r_k}{dt}(t)+\frac{dA}{dt}(t)\cdot\vec r_k(t)+\frac{d\vec
z}{dt}(t)\quad\quad\forall\, k\,=\,1,2,\ldots, n,
\end{cases}
\end{equation}
and thus also,
\begin{equation}\label{guigikvbvbGGjkilil11uhggyjijkjk11}
\begin{cases}
\vec l'_k(t')=A(t)\cdot\vec l_k(t)\quad\quad\forall\,
k\,=\,1,2,\ldots, n,\\
\frac{d\vec l'_k}{dt'}(t')=A(t)\cdot\frac{d\vec
l_k}{dt}(t)+\frac{dA}{dt}(t)\cdot\vec l_k(t)\quad\quad\forall\,
k\,=\,1,2,\ldots, n,
\\
\left(\vec r'(t')+s\vec l'_k(t')\right)=A(t)\cdot\left(\vec r(t)+s
\vec l_k(t)\right)+\vec z(t)\quad\quad\forall\, k\,=\,1,2,\ldots,
n\;\;\forall s,
\\
\left(\frac{d\vec r'}{dt'}(t')+s\frac{d\vec
l'_k}{dt'}(t')\right)=A(t)\cdot\left(\frac{d\vec
r}{dt}(t)+s\frac{d\vec
l_k}{dt}(t)\right)+\frac{dA}{dt}(t)\cdot\left(\vec r(t)+s\vec
l_k(t)\right)+\frac{d\vec z}{dt}(t),
\\
\left(\frac{d\vec r'}{dt'}(t')+s\frac{d\vec l'_k}{dt'}(t')-\vec
u'\left(\vec r'(t')+s\vec
l'_k(t'),t'\right)\right)=A(t)\cdot\left(\frac{d\vec
r}{dt}(t)+s\frac{d\vec l_k}{dt}(t)-\vec u\left(\vec r(t)+s\vec
l_k(t),t\right)\right).
\end{cases}
\end{equation}
Therefore, as before, by \er{guigikvbvbGGjkilil11uhggy11} and
\er{guigikvbvbGGjkilil11uhggyjijkjk11}, $\vec P(\vec x,t)$ and $\vec
M_{\vec u}(\vec x,t)$, defined either by
\er{hgkgjhfhjfjfiouykkhhhjjgjgkhkhkjujuyfffffffffffffffhjjhjhj11uiuyuyjii}
or by \er{hgkgjhfhjfjfiouykkhhhjjgjgkhkhkjujkjjkjkjk11}, clearly
satisfy the first two equations in \er{guigikvbvbGGjkilil11}.
Moreover, by the first two equations in \er{guigikvbvbGGjkilil11}
together with \er{hgkguiuioiioioioyyy} and the first equation in
\er{guigikvbvbGGjkilil11uhggy11} we deduce the third equation in
\er{guigikvbvbGGjkilil11}. Finally, as before, identities in either
\er{hgkgjhfhjfjfiouykkhhhjhjhjhjhhkjljljl;lk;lk;l;k;utityi11} or
\er{hgkgjhfhjfjfiouykkhhhjhjhjhjhhkjljljl;lk;lk;l;k;utityi11jijij}
are invariant under the change of cartesian coordinate systems,
provided we have \er{guigikvbvbGGjkilil11} and
\er{guigikvbvbGGjkililyjjkuhjhgf11}.
\end{proof}

Next consider a complex system containing $m$ totally neutral
simpler systems, such that for every $k=1,2,\ldots,m$ every simple
system consists of $n_k$ point charges
$\sigma^{(k)}_1,\ldots,\sigma^{(k)}_{n_k}$, with places $\vec
r^{(k)}_1(t),\ldots,\vec r^{(k)}_{n_k}(t)$ which are very close to
the place $\vec r^{(k)}(t)$ and velocities $\frac{d\vec
r^{(k)}_1}{dt}(t),\ldots, \frac{d\vec r^{(k)}_{n_k}}{dt}(t)$ which
are very close to the velocity $ \frac{d\vec r^{(k)}}{dt}(t)$,
satisfying
\begin{equation}\label{hgkgjhfhjfjfiouykkhhhjjgjgkhkhkjuj11yyttt11jkjik11}
\sum_{j=1}^{n_k}\sigma^{(k)}_j=0\quad\quad\forall k=1,2,\ldots,m.
\end{equation}
and such that all the quantities $\big|\vec l^{(k)}_j(t)\big|$
defined by
\begin{equation}\label{hgkgjhfhjfjfiouykkhhhjhjhjhjhh11ghghghf11kljjkkkhhuuugnbhnhkhhn11jkjk11hjghhghghghg11}
\vec l^{(k)}_j(t):=\vec r^{(k)}_j(t)-\vec
r^{(k)}(t)\quad\quad\forall\, k\,=\,1,2,\ldots, m,\;\;\forall\,
j\,=\,1,2,\ldots, n_k,
\end{equation}
are infinitely small. As an example of such a complex system we can
take a system containing $m$ point \underline{dipoles} or $m$
totally neutral molecules. Then, consistently with
\er{hgkgjhfhjfjfiouykkhhhjjgjgkhkhkjujuyfffffffffffffffhjjhjhji99u9u9j11}
we can write the Polarization and the Magnetization of the system
as:
\begin{multline}\label{hgkgjhfhjfjfiouykkhhhjjgjgkhkhkjujuyfffffffffffffffhjjhjhji99u9u9j11ggjgghgh11}
\vec P(\vec x,t)\approx
\sum_{k=1}^{m}\sum_{j=1}^{n_k}\sigma^{(k)}_j\vec l^{(k)}_j(t)\,\delta\left(\vec x-\vec r^{(k)}(t)\right)\quad\quad\quad\text{and}\quad\\
\vec M(\vec
x,t)\approx-\frac{1}{c}\left(\sum_{k=1}^{m}\sum_{j=1}^{n_k}\sigma^{(k)}_j
\frac{d\vec
r^{(k)}}{dt}(t)
\times\vec l^{(k)}_j(t)\,\delta\left(\vec x-\vec
r^{(k)}(t)\right)\right).
\end{multline}
\begin{remark}\label{gjgkgkhgghhg}
In the frames of Quantum Physics one can show that there is an
additional different term in the expression of the Magnetization
$\vec M(\vec x,t)$, arisen from the spin of charged particles.
However, it is impossible to describe this term in the frames of
Classical Mechanics.
\end{remark}

\subsection{General setting of Maxwell equations in the medium} Consider system
\er{MaxVacFull1bjkgjhjhgjaaaPPN}
in some inertial or non-inertial cartesian coordinate system inside
a dielectric and/or magnetic medium:
\begin{equation}\label{MaxVacFullnnnnGG}
\begin{cases}
curl_{\vec x} \vec H_0\equiv \frac{4\pi}{c}\left(\vec j+\vec
j_{mp}\right)+
\frac{1}{c}\frac{\partial \vec D_0}{\partial t}\quad\text{for}\;\;(\vec x,t)\in\R^3\times[t_0,+\infty),\\
div_{\vec x} \vec D_0\equiv 4\pi\left(\rho+\rho_p\right)\quad\quad\text{for}\;\;(\vec x,t)\in\R^3\times[t_0,+\infty),\\
curl_{\vec x} \vec E+\frac{1}{c}\frac{\partial \vec B}{\partial t}\equiv 0\quad\quad\text{for}\;\;(\vec x,t)\in\R^3\times[t_0,+\infty),\\
div_{\vec x} \vec B\equiv 0\quad\quad\text{for}\;\;(\vec
x,t)\in\R^3\times[t_0,+\infty),
\end{cases}
\end{equation}
where $\vec E$ is the electric field, $\vec B$ is the magnetic
field, $\vec v:=\vec v(\vec x,t)$ is the vectorial gravitational
potential, $\rho$ is the average (macroscopic) charge density,
$\rho_p$ is the density of the charge of polarization, $\vec j$ is
the average (macroscopic) current density, $\vec j_{mp}$ is the
density of the current of polarization and magnetization and
\begin{equation}\label{MaxVacFullnnnngkggjkklhGG}
\vec D_0:=\vec E+\frac{1}{c}\,\vec v\times \vec
B\quad\text{and}\quad
\vec H_0:=\vec B+\frac{1}{c}\,\vec v\times \vec D_0.
\end{equation}
Consistently with subsection \ref{huugftdtdtddtdt}, it is well known
from the Lorentz theory that in the case of a moving
dielectric/magnetic medium
\begin{equation}\label{PolarGG}
\rho_p=-div_{\vec x} \vec P\quad\text{and}\quad
\vec j_{mp}=\frac{\partial \vec P}{\partial t}+c\, curl_{\vec x}
\vec M
\end{equation}
where $\vec P:\R^3\times[t_0,+\infty)\to\R^3$ is the field of
polarization, $\vec M:\R^3\times[t_0,+\infty)\to\R^3$ is the field
of magnetization and $\vec u:=\vec u(\vec x,t)$ is the field of
velocities of the dielectric medium, see equation
\er{hgkgjhfhjfjfiouykkhhhjhjhjhjhhkjljljl;lk;lk;l;k;utityi11jijij}
in Proposition \ref{hjghfhfhf11} for the proof of \er{PolarGG}.
Thus, if we consider
\begin{equation}\label{OprdddGG}
\vec D:=\vec D_0+4\pi \vec P=\vec E+\frac{1}{c}\,\vec v\times \vec
B+4\pi \vec P,
\end{equation}
and
\begin{multline}\label{Oprddd1GG}
\vec H:=\vec H_0-4\pi \vec M
=\vec B+\frac{1}{c}\,\vec v\times \vec D_0
-4\pi \vec M
=\vec B
+\frac{1}{c}\,\vec v\times
\vec E+\frac{1}{c}\,\vec v\times\left(\frac{1}{c}\,\vec v\times \vec
B\right)-4\pi \vec M,
\end{multline}
we obtain the usual Maxwell equations of the form:
\begin{equation}\label{MaxMedFullGG}
\begin{cases}
curl_{\vec x} \vec H\equiv \frac{4\pi}{c}\vec j+
\frac{1}{c}\frac{\partial \vec D}{\partial t}\quad\text{for}\;\;(\vec x,t)\in\R^3\times[t_0,+\infty),\\
div_{\vec x} \vec D\equiv 4\pi\rho\quad\quad\text{for}\;\;(\vec x,t)\in\R^3\times[t_0,+\infty),\\
curl_{\vec x} \vec E+\frac{1}{c}\frac{\partial \vec B}{\partial t}\equiv 0\quad\quad\text{for}\;\;(\vec x,t)\in\R^3\times[t_0,+\infty),\\
div_{\vec x} \vec B\equiv 0\quad\quad\text{for}\;\;(\vec
x,t)\in\R^3\times[t_0,+\infty),
\end{cases}
\end{equation}
We call $\vec D$ by the electric displacement field and $\vec H$ by
the $\vec H$-magnetic field in a medium.

\subsection{Change of Non-inertial coordinate system} Consider the
change of certain non-inertial cartesian coordinate system $(*)$ to
another cartesian coordinate system $(**)$ of the form
\begin{equation*}
\begin{cases}
\vec x'=A(t)\cdot \vec x+\vec z(t),\\
t'=t,
\end{cases}
\end{equation*}
where $A(t)\in SO(3)$ is a rotation. Then, as before in
\er{yuythfgfyftydtydtydtyddyyyhhddhhhredPPN111hgghjg}, denoting
$\vec w(t)=\vec z'(t)$, we have the following relations between the
physical quantities in coordinate systems $(*)$ and $(**)$:
\begin{equation}\label{guigikvbvbGG}
\begin{cases}
\vec E'=A(t)\cdot\vec E-\frac{1}{c}\,\left(A'(t)\cdot\vec x+\vec
w(t)\right)\times \left(A(t)\cdot\vec B\right),\\
\vec B'=A(t)\cdot\vec B,\\
\vec D'_0=A(t)\cdot \vec D_0,\\
\vec H'_0=A(t)\cdot\vec H_0+\frac{1}{c}\,\left(A'(t)\cdot\vec x+\vec
w(t)\right)\times \left(A(t)\cdot\vec D_0\right),\\
\vec u'=A(t)\cdot \vec u+A'(t)\cdot\vec x+\vec w(t),\\
\vec P'=A(t)\cdot\vec P,\\
\vec M'=A(t)\cdot\vec M-\frac{1}{c}\,\left(A'(t)\cdot\vec x+\vec
w(t)\right)\times \left(A(t)\cdot\vec P\right),
\end{cases}
\end{equation}
see subsection \ref{huugftdtdtddtdt} for the justification of the
last two equalities in \er{guigikvbvbGG}. Plugging it into
\er{OprdddGG} and \er{Oprddd1GG} we deduce
\begin{equation}\label{OprdddkkkmGG}
\vec D':=\vec D'_0+4\pi \vec P'=A(t)\cdot\left(\vec D_0+4\pi \vec
P\right)=A(t)\cdot\vec D,
\end{equation}
and
\begin{multline}\label{Oprddd1kkkmGG}
\vec H':=\vec H'_0-4\pi \vec M'
=A(t)\cdot\vec H_0+\frac{1}{c}\,\left(A'(t)\cdot\vec x+\vec
w(t)\right)\times \left(A(t)\cdot\vec D_0\right)\\-4\pi
A(t)\cdot\vec M+\frac{4\pi}{c}\,\left(
A'(t)\cdot\vec x+\vec
w(t)\right)\times \left(A(t)\cdot\vec P\right)\\
=A(t)\cdot\left(\vec H_0-4\pi \vec M
\right)
+\frac{1}{c}\,\left(A'(t)\cdot\vec x+\vec
w(t)\right)\times \left(A(t)\cdot\left(\vec D_0+4\pi\vec
P\right)\right)\\= A(t)\cdot\vec H+\frac{1}{c}\,\left(A'(t)\cdot\vec
x+\vec w(t)\right)\times \left(A(t)\cdot\vec D\right) ,
\end{multline}
So the expressions of transformations under the change of
non-inertial cartesian coordinate system in a dielectric/magnetic
medium exactly the same as in the vacuum, i.e. having the form of
\begin{equation}\label{guigikvbvbggjklhjkkgjgGGGG}
\begin{cases}
\vec D'=A(t)\cdot \vec D\\
\vec B'=A(t)\cdot\vec B\\
\vec E'=A(t)\cdot\vec E-\frac{1}{c}\,\left(A'(t)\cdot\vec x+\vec
w(t)\right)\times \left(A(t)\cdot\vec B\right)\\
\vec H'=A(t)\cdot\vec H+\frac{1}{c}\,\left(A'(t)\cdot\vec x+\vec
w(t)\right)\times \left(A(t)\cdot\vec D\right).
\end{cases}
\end{equation}
Moreover, as before, the Maxwell equations in the medium of the form
\er{MaxMedFullGG} stays invariant under the change of inertial or
non-inertial cartesian coordinate system, provided we have
\er{guigikvbvbggjklhjkkgjgGGGG} and
\begin{equation}\label{hjjgkooo}
\vec j'=A(t)\cdot \vec j+\rho\,A'(t)\cdot\vec x+\rho\,\vec w(t).
\end{equation}

\subsection{Case of simplest dielectrics/magnetics}\label{hhhhhhijh}
It is well known that in the case of simplest resting homogenous
isotropic nonmagnetic dielectric medium we have
\begin{equation}\label{EBDHTrans444knbuihuigGGojkjkuukjljupp}
\begin{cases}
\vec P=\frac{n^2-1}{4\pi}\,\vec E,\\
\vec M=0,
\end{cases}
\end{equation}
where $n$ is a material coefficient (not necessary constant), called
\underline{refraction index}. Moreover, It is well known that in the
case of simplest resting homogenous isotropic dielectric medium with
certain magnetic properties we rewrite
\er{EBDHTrans444knbuihuigGGojkjkuukjljupp} as
\begin{equation}\label{EBDHTrans444knbuihuigGGjojojjjjoojj}
\begin{cases}
\vec P=\frac{n^2-1}{4\pi}\,\vec E-\kappa\vec D,\\
\vec M=\kappa\vec H,
\end{cases}
\end{equation}
where $n$ is the refraction index and $\kappa$ is an additional
material coefficient (not necessary constant). Either
\er{EBDHTrans444knbuihuigGGojkjkuukjljupp} or
\er{EBDHTrans444knbuihuigGGjojojjjjoojj} are valid only for points
and instants of time where the velocity of the medium $\vec u$
vanishes. In the case of nonmagnetic dielectric we just put
$\kappa=0$ and \er{EBDHTrans444knbuihuigGGjojojjjjoojj} becomes to
be the same as \er{EBDHTrans444knbuihuigGGojkjkuukjljupp}.

Next, in the case of simplest moving homogenous isotropic
nonmagnetic dielectric medium we can generalize
\er{EBDHTrans444knbuihuigGGojkjkuukjljupp} as:
\begin{equation}\label{EBDHTrans444knbuihuigGGojkjkuuku}
\begin{cases}
\vec P=\frac{n^2-1}{4\pi}\left(\vec E+\frac{1}{c}\,\vec u\times \vec B\right),\\
\vec M=-\frac{1}{c}\,\frac{n^2-1}{4\pi}\,\vec u\times \left(\vec
E+\frac{1}{c}\,\vec u\times \vec B\right)=-\frac{1}{c}\vec u\times
\vec P,
\end{cases}
\end{equation}
where $n$ is a material coefficient (not necessary constant), called
\underline{refraction index} and $\vec u$ is the velocity of the
medium. Moreover, in the case of simplest moving homogenous
isotropic dielectric medium with certain magnetic properties we can
generalize \er{EBDHTrans444knbuihuigGGjojojjjjoojj} as:
\begin{equation}\label{EBDHTrans444knbuihuigGG}
\begin{cases}
\vec P=\frac{n^2-1}{4\pi}\left(\vec E+\frac{1}{c}\,\vec u\times \vec B\right)-\kappa\vec D,\\
\vec M=-\frac{1}{c}\,\frac{n^2-1}{4\pi}\,\vec u\times \left(\vec
E+\frac{1}{c}\,\vec u\times \vec B\right)+\kappa\vec
H=-\frac{1}{c}\vec u\times \vec P+\kappa\left(\vec
H-\frac{1}{c}\,\vec u\times \vec D\right),
\end{cases}
\end{equation}
where $n$ is the refraction index and $\kappa$ is an additional
material coefficient (not necessary constant). In the case of
nonmagnetic dielectric we just put $\kappa=0$ and
\er{EBDHTrans444knbuihuigGG} becomes to be the same as
\er{EBDHTrans444knbuihuigGGojkjkuuku}.
Then, using \er{guigikvbvbGG} and \er{guigikvbvbggjklhjkkgjgGGGG},
it can be easily seen that the laws in either
\er{EBDHTrans444knbuihuigGGojkjkuuku} or
\er{EBDHTrans444knbuihuigGG} are invariant under the changes of
inertial or non-inertial cartesian coordinate system. Alternatively,
one can assume either \er{EBDHTrans444knbuihuigGGojkjkuukjljupp} and
\er{EBDHTrans444knbuihuigGGjojojjjjoojj} for the case $\vec u=0$ and
to postulate the invariance under the changes of inertial or
non-inertial cartesian coordinate system. Then, using
\er{guigikvbvbGG} and \er{guigikvbvbggjklhjkkgjgGGGG}, one can
definitely deduce either \er{EBDHTrans444knbuihuigGGojkjkuuku} or
\er{EBDHTrans444knbuihuigGG} in the case $\vec u\neq 0$, similarly
as it was done in section \ref{MaxRevsPPN}, for establishing the
relations $\vec D\sim \vec E$ and $\vec H\sim \vec B$ in the case
$\vec v\neq 0$.

Next, plugging \er{EBDHTrans444knbuihuigGG} into \er{OprdddGG} and
\er{Oprddd1GG} gives,
\begin{equation}\label{OprdddsimGG}
\vec D:=\vec E+\frac{1}{c}\,\vec v\times \vec B+4\pi
\left(\frac{n^2-1}{4\pi}\left(\vec E+\frac{1}{c}\,\vec u\times \vec
B\right)-\kappa\vec D\right),
\end{equation}
and
\begin{equation}\label{Oprddd1simGG}
\vec H: =\vec B
+\frac{1}{c}\,\vec v\times \vec E+\frac{1}{c}\,\vec
v\times\left(\frac{1}{c}\,\vec v\times \vec
B\right)-4\pi\left(-\frac{1}{c}\,\frac{n^2-1}{4\pi}\,\vec u\times
\left(\vec E+\frac{1}{c}\,\vec u\times \vec B\right)+\kappa\vec
H\right).
\end{equation}
We rewrite \er{OprdddsimGG} as:
\begin{equation}\label{OprdddsimGGffyzz}
\vec E=\frac{1+4\pi\kappa}{n^2}\,\vec
D-\frac{1}{c}\,\left(\frac{1}{n^2}\vec
v+\left(1-\frac{1}{n^2}\right)\,\vec u\right)\times \vec B,
\end{equation}
and by \er{OprdddsimGG} and \er{OprdddsimGGffyzz} we rewrite
\er{Oprddd1simGG} as:
\begin{multline}\label{Oprddd1simGGffyzzhhh}
\left(1+4\pi \kappa\right)\vec H =\vec B
+\frac{1}{c}\,\vec v\times \vec
E+\frac{1}{c}\,\left(n^2-1\right)\,\vec u\times \vec
E+\frac{1}{c}\,\vec v\times\left(\frac{1}{c}\,\vec v\times \vec
B\right)+\frac{1}{c}\,\left(n^2-1\right)\,\vec u\times
\left(\frac{1}{c}\,\vec u\times \vec B\right)\\=\vec B
+\frac{1}{c}\,\left(1+4\pi \kappa\right)\,\left(\frac{1}{n^2}\vec
v+\left(1-\frac{1}{n^2}\right)\,\vec u\right)\times \vec D
\\
-\frac{1}{c}\,\vec v\times
\left(\frac{1}{c}\,\left(\frac{1}{n^2}\vec
v+\left(1-\frac{1}{n^2}\right)\,\vec u\right)\times \vec
B\right)-\frac{1}{c}\,\left(n^2-1\right)\,\vec u\times
\left(\frac{1}{c}\,\left(\frac{1}{n^2}\vec
v+\left(1-\frac{1}{n^2}\right)\,\vec u\right)\times \vec
B\right)\\+\frac{1}{c}\,\vec v\times\left(\frac{1}{c}\,\vec v\times
\vec B\right)+\frac{1}{c}\,\left(n^2-1\right)\,\vec u\times
\left(\frac{1}{c}\,\vec u\times \vec B\right)=\\
\vec B
+\frac{1}{c}\,\left(1+4\pi \kappa\right)\,\left(\frac{1}{n^2}\vec
v+\left(1-\frac{1}{n^2}\right)\,\vec u\right)\times \vec D
+\frac{1}{c^2}\,\left(1-\frac{1}{n^2}\right)\,\left(\vec u-\,\vec
v\right)\times \left(\left(\vec u-\,\vec v\right)\times \vec
B\right)\,,
\end{multline}
so that
\begin{equation}\label{Oprddd1simGGffyzz}
\vec H=\frac{1}{1+4\pi \kappa}\vec
B+\frac{1}{c}\left(\frac{1}{n^2}\vec
v+\left(1-\frac{1}{n^2}\right)\,\vec u\right)\times \vec D
+\frac{1}{c^2}\,\frac{1}{1+4\pi
\kappa}\,\left(1-\frac{1}{n^2}\right)\,\left(\vec u-\,\vec
v\right)\times \left(\left(\vec u-\,\vec v\right)\times \vec
B\right).
\end{equation}
Thus denoting
\begin{equation}\label{hkjhjh}
\kappa_0=\frac{1}{1+4\pi
\kappa}\quad\text{and}\quad\gamma_0=\frac{1+4\pi\kappa}{n^2}=\frac{1}{\kappa_0n^2}\quad\text{so
that}\quad n=\frac{1}{\sqrt{\gamma_0\kappa_0}}
\end{equation}
and defining the speed-like vector field
\begin{equation}\label{OprdddsimGGffyhjyhhzz}
\vec {\tilde u}:=\left(\frac{1}{n^2}\vec
v+\left(1-\frac{1}{n^2}\right)\,\vec
u\right)=\left(\gamma_0\kappa_0\vec
v+\left(1-\gamma_0\kappa_0\right)\,\vec u\right),
\end{equation}
by \er{OprdddsimGGffyzz} and \er{Oprddd1simGGffyzz}  we deduce
\begin{equation}\label{OprdddsimGGffykkzz}
\vec E=\gamma_0\vec D-\frac{1}{c}\,\vec {\tilde u}\times \vec B,
\end{equation}
and
\begin{equation}\label{Oprddd1simGGffykkzz}
\vec H=\kappa_0\vec B+\frac{1}{c}\,\vec {\tilde u}\times \vec
D+\kappa_0\frac{(1-\gamma_0\kappa_0)}{c^2}\,(\vec u-\vec
v)\times\left(\left(\vec u-\vec v\right)\times \vec B\right),
\end{equation}
where we call $\gamma_0$ and $\kappa_0$ dielectric and magnetic
permeability of the medium. Thus, by \er{MaxMedFullGG},
\er{OprdddsimGGffyhjyhhzz} \er{OprdddsimGGffykkzz} and
\er{Oprddd1simGGffykkzz} we have
\begin{equation}\label{MaxMedFullGGffykkzz}
\begin{cases}
curl_{\vec x} \vec H=\frac{4\pi}{c}\vec j+
\frac{1}{c}\frac{\partial \vec D}{\partial t},\\
div_{\vec x} \vec D=4\pi\rho,\\
curl_{\vec x} \vec E+\frac{1}{c}\frac{\partial \vec B}{\partial t}=0,\\
div_{\vec x} \vec B=0,\\
\vec E=\gamma_0\vec D-\frac{1}{c}\,\vec {\tilde u}\times \vec B,\\
\vec H=\kappa_0\vec B+\frac{1}{c}\,\vec {\tilde u}\times \vec
D+\frac{\kappa_0(1-\gamma_0\kappa_0)}{c^2}\,(\vec u-\vec
v)\times\left(\left(\vec
u-\vec v\right)\times \vec B\right),\\
\vec {\tilde u}:=\left(\gamma_0\kappa_0\vec
v+\left(1-\gamma_0\kappa_0\right)\vec u\right),
\end{cases}
\end{equation}
where $\vec {\tilde u}$ is a speed-like vector field that we call
the optical displacement of the moving medium. Note that for the
case $\gamma_0=1$ and $\kappa_0=1$, the system
\er{MaxMedFullGGffykkzz} is exactly the same as the corresponding
system in the vacuum. The equations in \er{MaxMedFullGGffykkzz} take
much simpler forms in the case where the quantity
\begin{equation}\label{OprdddsimGGffyhjyhhtygrffgfzz}
\frac{|\kappa_0||1-\gamma_0\kappa_0|\cdot|\vec u-\vec v|^2}{c^2}\ll
1
\end{equation} is
negligible, that happens if either the absolute value of the
difference between the medium velocity and vectorial gravitational
potential is much less then the constant $c$ or
$\gamma_0\kappa_0=\frac{1}{n^2}$ is close to the value $1$. Indeed,
in this case, instead of \er{OprdddsimGGffykkzz} and
\er{Oprddd1simGGffykkzz} we obtain the following relations:
\begin{align}\label{OprdddsimsimsimGG}
\vec E=\gamma_0\vec D-\frac{1}{c}\,\vec {\tilde u}\times \vec B,\\
\label{Oprddd1simsimsimGG}
\vec H=\kappa_0\vec B+\frac{1}{c}\,\vec {\tilde u}\times \vec D.
\end{align}
As a consequence we obtain the full system of Maxwell equations in
the medium:
\begin{equation}\label{MaxMedFullGGffykkhjhhzz}
\begin{cases}
curl_{\vec x} \vec H=\frac{4\pi}{c}\vec j+
\frac{1}{c}\frac{\partial \vec D}{\partial t},\\
div_{\vec x} \vec D=4\pi\rho,\\
curl_{\vec x} \vec E+\frac{1}{c}\frac{\partial \vec B}{\partial t}=0,\\
div_{\vec x} \vec B=0,\\
\vec E=\gamma_0\vec D-\frac{1}{c}\,\vec {\tilde u}\times \vec B,\\
\vec H=\kappa_0\vec B+\frac{1}{c}\,\vec {\tilde u}\times \vec D,\\
\vec {\tilde u}=\left(\gamma_0\kappa_0\vec
v+\left(1-\gamma_0\kappa_0\right)\vec u\right),
\end{cases}
\end{equation}
where $\vec {\tilde u}$ is the speed-like vector field and
$\gamma_0$ and $\kappa_0$ are dielectric and magnetic permeability
of the medium. Note that \er{MaxMedFullGGffykkhjhhzz} is analogous
to the system of Maxwell equations in the vacuum and it is also
invariant under the change of inertial or non-inertial cartesian
coordinate system, provided that under this transformation we have
\er{guigikvbvbggjklhjkkgjgGGGG}.

\subsection{Case of anisotropic dielectrics/magnetics} It is well
known that in the case the simplest anisotropic dielectrics and/or
magnetics we have the following generalization of
\er{EBDHTrans444knbuihuigGG}:
\begin{equation}\label{EBDHTrans444knbuihuigGGaa}
\begin{cases}
\vec P=\Gamma\cdot\left(\vec E+\frac{1}{c}\,\vec u\times \vec B\right)-\Upsilon\cdot\vec D,\\
\vec M=
-\frac{1}{c}\vec u\times \vec P+\Upsilon\cdot\left(\vec
H-\frac{1}{c}\,\vec u\times \vec D\right),
\end{cases}
\end{equation}
where $\Gamma\in \mathbb{R}^{3\times 3}$ and $\Upsilon\in
\mathbb{R}^{3\times 3}$ are matrix-valued fields.
Then, using \er{guigikvbvbGG} it can be easily seen that the laws in
\er{EBDHTrans444knbuihuigGGaa} are invariant under the changes of
inertial or non-inertial cartesian coordinate system, provided that
$\Gamma$ and $\Upsilon$ are \underline{proper} matrix fields (see
Definition \ref{bggghghgj}).

\subsection{Ohm's Law in a conducting medium} It is well known that
the Ohm's Law in a conducting medium has the form
\begin{equation}\label{vjhfhjtjhjhuyyiyGG}
\vec j-\rho\,\vec u=\varepsilon\left(\vec E+\frac{1}{c}\,\vec
u\times \vec B\right),
\end{equation}
where $\vec u$ is the velocity of the medium and $\varepsilon$ is a
material coefficient. As before, using
\er{guigikvbvbGG}, it can be easily seen that the Ohm's Law is
invariant under the changes of inertial or non-inertial cartesian
coordinate system.

Furthermore, it is well known that in the case of the strong
magnetic field the modification of the the Ohm's Law including the
Hall effect has the following form:
\begin{equation}\label{vjhfhjtjhjhuyyiyGGjlj}
\vec j-\rho\,\vec u=\varepsilon\left(\vec E+\frac{1}{c}\,\vec
u\times \vec B\right)-\varsigma\left(\vec j-\rho\,\vec
u\right)\times \vec B,
\end{equation}
where $\varsigma$ is a material coefficient. Then, as before, using
\er{guigikvbvbGG}, it can be easily seen that the generalized Ohm's
Law \er{vjhfhjtjhjhuyyiyGGjlj}, including the Hall effect, is
invariant under the changes of inertial or non-inertial cartesian
coordinate system.

Next, it is well known that, in the case of the anisotropic
conducting medium, the Ohm's Law has the following form,
generalizing \er{vjhfhjtjhjhuyyiyGG}:
\begin{equation}\label{vjhfhjtjhjhuyyiyGGaa}
\vec j-\rho\,\vec u=\Xi\cdot\left(\vec E+\frac{1}{c}\,\vec u\times
\vec B\right),
\end{equation}
where $\Xi\in \mathbb{R}^{3\times 3}$ is a matrix-valued field. As
before, using
\er{guigikvbvbGG}, it can be easily seen that
\er{vjhfhjtjhjhuyyiyGGaa} is invariant under the changes of inertial
or non-inertial cartesian coordinate system, provided that $\Xi$ is
a \underline{proper} matrix field.

Finally, the generalization of \er{vjhfhjtjhjhuyyiyGGjlj} to
anisotropic mediums is the following:
\begin{equation}\label{vjhfhjtjhjhuyyiyGGjljaa}
\vec j-\rho\,\vec u=\Xi\cdot\left(\vec E+\frac{1}{c}\,\vec u\times
\vec B\right)-\Pi\cdot\left(\left(\vec j-\rho\,\vec u\right)\times
\vec B\right),
\end{equation}
where $\Xi\in \mathbb{R}^{3\times 3}$ and $\Pi\in
\mathbb{R}^{3\times 3}$  are matrix-valued fields. As before, using
\er{guigikvbvbGG}, it can be easily seen that
\er{vjhfhjtjhjhuyyiyGGjljaa} is invariant under the changes of
inertial or non-inertial cartesian coordinate system, provided that
$\Xi$ and $\Pi$ are \underline{proper} matrix fields.

\subsection{Optical dispersion in moving
mediums}\label{hugygyggygygyfgy} Remind that if  $U$ is a real
valued field, then we can write the field $U$ as a Fourier's
Transform on the time variable:
\begin{equation}\label{MaxVacFullPPNmmmffffffhhtygghGGGGyuhggghghghffggggjj11}
U(\vec x,t)=Re\left\{2\int_{0}^{+\infty} \hat U(\vec
x,\omega)e^{i\omega t}d\omega\right\}\quad\text{where}\quad\hat
U(\vec x,\omega):=\frac{1}{2\pi}\int_{-\infty}^{+\infty} U(\vec
x,t)e^{-i\omega t}dt\,.
\end{equation}
Thus, since $\vec E$, $\vec B$, $\vec D$ and $\vec H$ are real
vectors, we can write them as a Fourier's Transform on the time
variable:
\begin{align}\label{MaxVacFullPPNmmmffffffhhtygghGGGGyuhggghghghffggggjj11jhhhjj11hgjgj11ghfh11}
\vec E(\vec x,t)=Re\left\{2\int_{0}^{+\infty}  {\vec {\hat
E}(\omega,\vec x)}e^{i\omega t}d\omega\right\}\quad\text{where}\quad
{\vec {\hat E}(\omega,\vec
x)}:=\frac{1}{2\pi}\int_{-\infty}^{+\infty}
\vec E(\vec x,t)e^{-i \omega t}dt\,,\\
\label{MaxVacFullPPNmmmffffffhhtygghGGGGyuhggghghghffggggjj11jhhhjj11ihhgjjj11mbbmbm11}
\vec B(\vec x,t)=Re\left\{2\int_{0}^{+\infty}  {\vec {\hat
B}(\omega,\vec x)}e^{i\omega t}d\omega\right\}\quad\text{where}\quad
{\vec {\hat B}(\omega,\vec
x)}:=\frac{1}{2\pi}\int_{-\infty}^{+\infty} \vec B(\vec x,t)e^{-i
\omega
t}dt\,,\\
\label{MaxVacFullPPNmmmffffffhhtygghGGGGyuhggghghghffggggjj11jhhhjj11hgjgj11ghfh11D}
\vec D(\vec x,t)=Re\left\{2\int_{0}^{+\infty}  {\vec {\hat
D}(\omega,\vec x)}e^{i\omega t}d\omega\right\}\quad\text{where}\quad
{\vec {\hat D}(\omega,\vec
x)}:=\frac{1}{2\pi}\int_{-\infty}^{+\infty}
\vec D(\vec x,t)e^{-i \omega t}dt\,,\\
\label{MaxVacFullPPNmmmffffffhhtygghGGGGyuhggghghghffggggjj11jhhhjj11ihhgjjj11mbbmbm11H}
\vec H(\vec x,t)=Re\left\{2\int_{0}^{+\infty}  {\vec {\hat
H}(\omega,\vec x)}e^{i\omega t}d\omega\right\}\quad\text{where}\quad
{\vec {\hat H}(\omega,\vec
x)}:=\frac{1}{2\pi}\int_{-\infty}^{+\infty} \vec H(\vec x,t)e^{-i
\omega t}dt\,,
\end{align}
where given $\vec a=(a_1,a_2,a_3)\in \mathbb{C}^3$ we denote
$Re\,\{\vec a\}=\left(Re\,\{a_1\},Re\,\{a_2\},Re\,\{a_3\}\right)\in
\mathbb{R}^3$. Then, it is well known that in the case of the
simplest optical dispersion in the resting medium with $\vec u\equiv
0$ we have the following generalization of
\er{EBDHTrans444knbuihuigGGjojojjjjoojj}:
\begin{equation}\label{EBDHTrans444knbuihuigGGjj11jgggggjj11hhhjjj11gvghhghgjj11jgjggj11hghffh11}
\begin{cases}
\vec P\left(\vec x,t\right)=Re\left\{2\int_{0}^{+\infty}
{\frac{n^2\left(\omega,\vec x,t\right)-1}{4\pi}\,\vec {\hat
E}\left(\omega,\vec
 x\right)}e^{i\omega t}d\omega\right\}-Re\left\{2\int_{0}^{+\infty}
{\kappa\left(\omega,\vec x,t\right)\vec {\hat D}\left(\omega,\vec
x\right)}e^{i\omega t}d\omega\right\},\\
\vec M\left(\vec x,t\right)=Re\left\{2\int_{0}^{+\infty}
{\kappa\left(\omega,\vec x,t\right)\vec {\hat H}\left(\omega,\vec
x\right)}e^{i\omega t}d\omega\right\},
\end{cases}
\end{equation}
where $\vec {\hat E}(\omega,\vec x)$, $\vec {\hat B}(\omega,\vec
x)$, $\vec {\hat D}(\omega,\vec x)$ and $\vec {\hat H}(\omega,\vec
x)$ are Fourier's Transforms by time of $\vec E(\vec x,t)$, $\vec
B(\vec x,t)$, $\vec D(\vec x,t)$ and $\vec H(\vec x,t)$, given by
the second equalities of
\er{MaxVacFullPPNmmmffffffhhtygghGGGGyuhggghghghffggggjj11jhhhjj11hgjgj11ghfh11},
\er{MaxVacFullPPNmmmffffffhhtygghGGGGyuhggghghghffggggjj11jhhhjj11ihhgjjj11mbbmbm11},
\er{MaxVacFullPPNmmmffffffhhtygghGGGGyuhggghghghffggggjj11jhhhjj11hgjgj11ghfh11D}
and
\er{MaxVacFullPPNmmmffffffhhtygghGGGGyuhggghghghffggggjj11jhhhjj11ihhgjjj11mbbmbm11H}
and the quantities $n\left(\omega,\vec x,t\right)$ and
$\kappa\left(\omega,\vec x,t\right)$ in
\er{EBDHTrans444knbuihuigGGjj11jgggggjj11hhhjjj11gvghhghgjj11jgjggj11hghffh11}
are assumed to be \underline{complex} in general and assumed to
depend on $\omega$ in addition to the dependence on $(\vec x,t)$.

We would like to obtain the law being an analog of
\er{EBDHTrans444knbuihuigGGjj11jgggggjj11hhhjjj11gvghhghgjj11jgjggj11hghffh11}
in the case of moving medium (i.e. $\vec u\neq0$), which is
invariant under the change of inertial or non-inertial cartesian
coordinate system. In addition to the invariance under the change of
cartesian coordinate systems and the equivalence to
\er{EBDHTrans444knbuihuigGGjj11jgggggjj11hhhjjj11gvghhghgjj11jgjggj11hghffh11}
in the particular case $\vec u\equiv 0$, we also need to assume that
in the case of an arbitrary moving transparent medium without
dispersion our law is equivalent to \er{EBDHTrans444knbuihuigGG}.

Then, in some cartesian coordinate system consider a motion of some
continuum medium occupying a region $\Omega\subset\mathbb{R}^3$ at
some fixed instant of time $t=t_0$ and having the velocity field
$\vec u:=\vec u(\vec x,t)$. Next, as before, let $\vec r(t,\vec
y):\mathbb{R}\times\Omega\to\mathbb{R}^3$ be a solution of
\er{hjgghhghhffhkk} i.e. it satisfies the following initial value
problem for an ordinary differential equation:
\begin{equation}\label{hjgghhghhffhkkjj11}
\begin{cases}
\frac{\partial\vec r}{\partial t}(t,\vec y)=\vec u\left(\vec
r(t,\vec y),t\right)\quad\quad\forall t\in\mathbb{R},\;\;\forall\vec
y\in\Omega\\
\vec r(t_0,\vec y)=\vec y\quad\quad\forall\vec y\in\Omega.
\end{cases}
\end{equation}
Then, clearly $\vec x=\vec r(t,\vec y)$ stands for the spatial
coordinates at the instant of time $t$ of the parcel of continuum,
having initial coordinates $\vec y$. Thus
by \er{hjgghhghhffhkkihiyjhkpkpkbb} we deduce that
$det\left\{d_{\vec y}\vec r(t,\vec y)\right\}\neq 0$ for every
instant of time $t$ and so, for the given instant of time $t$ the
mapping $\vec x=\vec r(t,\vec y)$ is locally invertible i.e. the
equation $\vec x=\vec r(t,\vec y)$ can be resolved in $\vec y$. Thus
there exists a regular mapping $\vec y=\vec f(\vec x,t)$, such that
\begin{equation}\label{hjgghhghhffhkkihiyjhjkhjhjhjj11}
\vec f\left(\vec r(t,\vec y),t\right)=\vec y\quad\forall \vec
y\in\Omega\quad\quad\text{and}\quad\quad\vec r\left(t,\vec f(\vec
x,t)\right)=\vec x.
\end{equation}
Then, clearly $\vec y=\vec f(\vec x,t)$ stands for the initial
coordinates at the time $t_0$ of the parcel of continuum, having
coordinates $\vec x$ at the instant of time $t$. Next $\vec f(\vec
x,t)$ satisfies:
\begin{equation}\label{hjgghhghhffhkkmkljjnnnnjj11}
\begin{cases}
\frac{\partial\vec f}{\partial t}\left(\vec x,t\right)+d_{\vec
x}\vec f\left(\vec x,t\right)\cdot\vec u\left(\vec x,t\right)=0
\\
\vec f(\vec x,t_0)=\vec x.
\end{cases}
\end{equation}
Next consider,
\begin{equation}\label{EBDHTrans444knbuihuigGGjj11ghthgtjkjj11}
\begin{cases}
\vec E^*(t,\vec y):=\left\{d_{\vec y}\vec r\left(t,\vec
y\right)\right\}^T\cdot\left(\vec E\left(\vec r(t,\vec
y),t\right)+\frac{1}{c}\,\vec u\left(\vec r(t,\vec y),t\right)\times
\vec B\left(\vec r(t,\vec y),t\right)\right)\quad\quad\forall
t\in\mathbb{R},\;\;\forall\vec
y\in\Omega,\\
\vec B^*(t,\vec y):=\left\{d_{\vec y}\vec r\left(t,\vec
y\right)\right\}^T\cdot\vec B\left(\vec r(t,\vec
y),t\right)\quad\quad\forall t\in\mathbb{R},\;\;\forall\vec
y\in\Omega,\\
\vec D^*(t,\vec y):=\left\{d_{\vec y}\vec r\left(t,\vec
y\right)\right\}^T\cdot\vec D\left(\vec r(t,\vec
y),t\right)\quad\quad\forall t\in\mathbb{R},\;\;\forall\vec
y\in\Omega,\\
\vec H^*(t,\vec y):=\left\{d_{\vec y}\vec r\left(t,\vec
y\right)\right\}^T\cdot\left(\vec H\left(\vec r(t,\vec
y),t\right)-\frac{1}{c}\,\vec u\left(\vec r(t,\vec y),t\right)\times
\vec D\left(\vec r(t,\vec y),t\right)\right)\quad\quad\forall
t\in\mathbb{R},\;\;\forall\vec y\in\Omega,
\end{cases}
\end{equation}
where $\vec r(t,\vec y)$ is given by \er{hjgghhghhffhkkjj11}.
Then, since $\vec E^*$, $\vec B^*$, $\vec D^*$ and $\vec H^*$ are
real vectors, we can write them as a Fourier's Transform on the time
variable:
\begin{align}
\label{MaxVacFullPPNmmmffffffhhtygghGGGGyuhggghghghffggggjj11jhhhjj11}
\vec E^*(t,\vec y)=Re\left\{2\int_{0}^{+\infty}  {\vec {\hat
E^*}(\tilde\omega,\vec y)}e^{i\tilde \omega
t}d\tilde\omega\right\}\quad\text{where}\quad {\vec {\hat
E^*}(\tilde\omega,\vec y)}:=\frac{1}{2\pi}\int_{-\infty}^{+\infty}
\vec E^*(t,\vec y)e^{-i\tilde \omega t}dt\,,\\
\label{MaxVacFullPPNmmmffffffhhtygghGGGGyuhggghghghffggggjj11jhhhjj11ihhgjjj11}
\vec B^*(t,\vec y)=Re\left\{2\int_{0}^{+\infty}  {\vec {\hat
B^*}(\tilde\omega,\vec y)}e^{i\tilde \omega
t}d\tilde\omega\right\}\quad\text{where}\quad {\vec {\hat
B^*}(\tilde\omega,\vec y)}:=\frac{1}{2\pi}\int_{-\infty}^{+\infty}
\vec B^*(t,\vec y)e^{-i\tilde \omega t}dt\,,\\
\label{MaxVacFullPPNmmmffffffhhtygghGGGGyuhggghghghffggggjj11jhhhjj11D}
\vec D^*(t,\vec y)=Re\left\{2\int_{0}^{+\infty}  {\vec {\hat
D^*}(\tilde\omega,\vec y)}e^{i\tilde \omega
t}d\tilde\omega\right\}\quad\text{where}\quad {\vec {\hat
D^*}(\tilde\omega,\vec y)}:=\frac{1}{2\pi}\int_{-\infty}^{+\infty}
\vec D^*(t,\vec y)e^{-i\tilde \omega t}dt\,,\\
\label{MaxVacFullPPNmmmffffffhhtygghGGGGyuhggghghghffggggjj11jhhhjj11ihhgjjj11H}
\vec H^*(t,\vec y)=Re\left\{2\int_{0}^{+\infty}  {\vec {\hat
H^*}(\tilde\omega,\vec y)}e^{i\tilde \omega
t}d\tilde\omega\right\}\quad\text{where}\quad {\vec {\hat
H^*}(\tilde\omega,\vec y)}:=\frac{1}{2\pi}\int_{-\infty}^{+\infty}
\vec H^*(t,\vec y)e^{-i\tilde \omega t}dt\,.
\end{align}
Then we write the generalization of
\er{EBDHTrans444knbuihuigGGjj11jgggggjj11hhhjjj11gvghhghgjj11jgjggj11hghffh11}
as:
\begin{multline}\label{EBDHTrans444knbuihuigGGjj11jgggggjj11hhhjjj11}
\vec P\left(\vec r(t,\vec y),t\right)= Re\left\{2\int_{0}^{+\infty}
{\frac{n^2\left(\tilde\omega,\vec r(t,\vec
y),t\right)-1}{4\pi}\,\left(\left(\left\{d_{\vec y}\vec
r\left(t,\vec y\right)\right\}^T\right)^{-1}\cdot\vec {\hat
E^*}\left(\omega,\vec
 x\right)\right)}e^{i\omega t}d\omega\right\}\\-Re\left\{2\int_{0}^{+\infty}
{\kappa\left(\tilde\omega,\vec r(t,\vec
y),t\right)\left(\left(\left\{d_{\vec y}\vec r\left(t,\vec
y\right)\right\}^T\right)^{-1}\cdot\vec {\hat D^*}\left(\omega,\vec
x\right)\right)}e^{i\omega t}d\omega\right\}\\
\quad\text{and}\quad\quad\quad \vec M\left(\vec r(t,\vec
y),t\right)+\frac{1}{c}\vec u\left(\vec
r(t,\vec y),t\right)\times \vec P\left(\vec r(t,\vec y),t\right)=\\
Re\left\{2\int_{0}^{+\infty} {\kappa\left(\tilde\omega,\vec r(t,\vec
y),t\right)\left(\left(\left\{d_{\vec y}\vec r\left(t,\vec
y\right)\right\}^T\right)^{-1}\cdot\vec {\hat H^*}(\tilde\omega,\vec
y)\right)}e^{i\tilde \omega t}d\tilde\omega\right\},
\end{multline}
i.e.
\begin{multline}\label{EBDHTrans444knbuihuigGGjj11jgggggjj11hhhjjj11gvghhghgjj11}
\vec P\left(\vec x,t\right)=Re\left\{2\int_{0}^{+\infty}
{\frac{n^2\left(\tilde\omega,\vec
x,t\right)-1}{4\pi}\,\left(\left\{d_{\vec x}\vec f\left(\vec
x,t\right)\right\}^T\cdot\vec {\hat E^*}\left(\tilde\omega,\vec
f\left(\vec x,t\right)\right)\right)}e^{i\tilde \omega
t}d\tilde\omega\right\}\\-Re\left\{2\int_{0}^{+\infty}
{\kappa\left(\tilde\omega,\vec x,t\right)\left(\left\{d_{\vec x}\vec
f\left(\vec x,t\right)\right\}^T\cdot\vec {\hat
D^*}\left(\tilde\omega,\vec f\left(\vec
x,t\right)\right)\right)}e^{i\tilde \omega t}d\tilde\omega\right\}\quad\quad\text{and} \\
\vec M\left(\vec x,t\right)+\frac{1}{c}\vec u(\vec x,t)\times \vec
P(\vec x,t)=Re\left\{2\int_{0}^{+\infty}
{\kappa\left(\tilde\omega,\vec x,t\right)\left(\left\{d_{\vec x}\vec
f\left(\vec x,t\right)\right\}^T\cdot\vec {\hat
H^*}\left(\tilde\omega,\vec f\left(\vec
x,t\right)\right)\right)}e^{i\tilde \omega t}d\tilde\omega\right\}.
\end{multline}
The quantities $n\left(\tilde\omega,\vec x,t\right)$ and
$\kappa\left(\tilde\omega,\vec x,t\right)$ in
\er{EBDHTrans444knbuihuigGGjj11jgggggjj11hhhjjj11} and
\er{EBDHTrans444knbuihuigGGjj11jgggggjj11hhhjjj11gvghhghgjj11} are
assumed to be \underline{complex} in general and assumed to depend
on $\tilde\omega$ in addition to the dependence on $(\vec x,t)$.

In particular, in the case $\vec u\equiv 0$ we clearly have $\vec
r(t,\vec y)\equiv\vec y$, $\vec f\left(\vec x,t\right)\equiv\vec x$
and $\vec E^*(t,\vec x)=\vec E(\vec x,t)$, $\vec B^*(t,\vec x)=\vec
B(\vec x,t)$, $\vec D^*(t,\vec x)=\vec D(\vec x,t)$, $\vec
H^*(t,\vec x)=\vec H(\vec x,t)$ and therefore,
\er{EBDHTrans444knbuihuigGGjj11jgggggjj11hhhjjj11} and
\er{EBDHTrans444knbuihuigGGjj11jgggggjj11hhhjjj11gvghhghgjj11}
coincide with
\er{EBDHTrans444knbuihuigGGjj11jgggggjj11hhhjjj11gvghhghgjj11jgjggj11hghffh11}.

On the other hand, in the case of an arbitrary moving transparent
medium without dispersion i.e. in the case where we assume that
$n:=n\left(\vec x,t\right)$ and $\kappa:=\kappa\left(\vec
x,t\right)$ are independent of the argument $\tilde\omega$ and
moreover we assume them to be \underline{real}, by the properties of
Fourier's Transform we rewrite
\er{EBDHTrans444knbuihuigGGjj11jgggggjj11hhhjjj11} as:
\begin{multline}\label{EBDHTrans444knbuihuigGGjj11jgggggjj11hhhjjj11jhjgj11}
\vec P\left(\vec r(t,\vec y),t\right)= \frac{n^2\left(\vec r(t,\vec
y),t\right)-1}{4\pi}\,\left(\left\{d_{\vec y}\vec r\left(t,\vec
y\right)\right\}^T\right)^{-1}\cdot\, Re\left\{2\int_{0}^{+\infty}
{\vec {\hat E^*}(\tilde\omega,\vec y)}e^{i\tilde \omega
t}d\tilde\omega\right\}\\-\kappa\left(\vec r(t,\vec
y),t\right)\left(\left\{d_{\vec y}\vec r\left(t,\vec
y\right)\right\}^T\right)^{-1}\cdot\,Re\left\{2\int_{0}^{+\infty}
{\vec {\hat D^*}(\tilde\omega,\vec y)}e^{i\tilde \omega
t}d\tilde\omega\right\}\\=\gamma\left(\vec r(t,\vec
y),t\right)\left(\left\{d_{\vec y}\vec r\left(t,\vec
y\right)\right\}^T\right)^{-1}\cdot\vec E^*(t,\vec
y)-\kappa\left(\vec r(t,\vec y),t\right)\left(\left\{d_{\vec y}\vec
r\left(t,\vec
y\right)\right\}^T\right)^{-1}\cdot\vec D^*(t,\vec y)\quad\quad\text{and}\\
\vec M\left(\vec r(t,\vec y),t\right)+\frac{1}{c}\vec u\left(\vec
r(t,\vec y),t\right)\times \vec P\left(\vec r(t,\vec y),t\right)=\\
\kappa\left(\vec r(t,\vec y),t\right)\left(\left\{d_{\vec y}\vec
r\left(t,\vec
y\right)\right\}^T\right)^{-1}\cdot\,Re\left\{2\int_{0}^{+\infty}
{\vec {\hat H^*}(\tilde\omega,\vec y)}e^{i\tilde \omega
t}d\tilde\omega\right\}\\=\kappa\left(\vec r(t,\vec
y),t\right)\left(\left\{d_{\vec y}\vec r\left(t,\vec
y\right)\right\}^T\right)^{-1}\cdot\vec H^*(t,\vec y).
\end{multline}
Thus, inserting  \er{EBDHTrans444knbuihuigGGjj11ghthgtjkjj11} into
\er{EBDHTrans444knbuihuigGGjj11jgggggjj11hhhjjj11jhjgj11} we deduce
\begin{multline}\label{EBDHTrans444knbuihuigGGjj11jgggggjj11hhhjjj11jhjgj11jggjgg11}
\vec P\left(\vec r(t,\vec y),t\right)= \frac{n^2\left(\vec r(t,\vec
y),t\right)-1}{4\pi}\,\left(\vec E\left(\vec r(t,\vec
y),t\right)+\frac{1}{c}\,\vec u\left(\vec r(t,\vec y),t\right)\times
\vec B\left(\vec r(t,\vec y),t\right)\right)\\-\kappa\left(\vec
r(t,\vec y),t\right)\vec D\left(\vec r(t,\vec y),t\right)
\quad\quad\text{and}\quad\quad \\
\vec M\left(\vec r(t,\vec y),t\right)+\frac{1}{c}\vec u\left(\vec
r(t,\vec y),t\right)\times \vec P\left(\vec r(t,\vec y),t\right)=
\\
\kappa\left(\vec r(t,\vec y),t\right)\left(\vec H\left(\vec
r(t,\vec y),t\right)-\frac{1}{c}\,\vec u\left(\vec r(t,\vec
y),t\right)\times \vec D\left(\vec r(t,\vec y),t\right)\right).
\end{multline}
I.e.
\begin{equation}\label{EBDHTrans444knbuihuigGGjj11gjjggjg11ghhfgg11}
\begin{cases}
\vec P(\vec x,t)=\frac{n^2\left(\vec x,t\right)-1}{4\pi}\,\left(\vec
E(\vec x,t)+\frac{1}{c}\,\vec u(\vec x,t)\times \vec B(\vec
x,t)\right)-\kappa(\vec x,t)\,\vec D(\vec
x,t),\\
\vec M(\vec x,t)=\kappa(\vec x,t)\,\left(\vec H(\vec
x,t)-\frac{1}{c}\,\vec u(\vec x,t)\times \vec D(\vec
x,t)\right)-\frac{1}{c}\vec u\left(\vec x,t\right)\times \vec
P\left(\vec x,t\right),
\end{cases}
\end{equation}
consistently with \er{EBDHTrans444knbuihuigGG}.

Next, as before, assume that the change of some non-inertial
cartesian system $(*)$ of coordinates to another cartesian system
$(**)$ of coordinates is of the form:
\begin{equation}\label{jjjjkkkk11}
\begin{cases}
\vec x'=A(t)\cdot\vec x+\vec z(t),\\
t'=t,
\end{cases}
\end{equation}
where $A(t)\in SO(3)$ is a rotation. Then, by \er{jjjjkkkkkhhhhjhjh}
the law of transformation of the Lagrangian coordinates $(\vec
y,t)\to(\vec y',t')$, consistent with \er{jjjjkkkk11}, is the
following:
\begin{equation}\label{jjjjkkkkkhhhhjhjh11}
\begin{cases}
\vec y'=A(t_0)\cdot\vec y+\vec z(t_0),\\
t'=t.
\end{cases}
\end{equation}
Moreover, by \er{hjgjgjgkllujkjjjkhhjh} we have
\begin{equation}\label{hjgjgjgkllujkjjjkhhjh11}
\vec r'(t',\vec y')=\vec r'\left(t,A(t_0)\cdot\vec y+\vec
z(t_0)\right)=A(t)\cdot\vec r(t,\vec y)+\vec z(t),
\end{equation}
and by \er{hjgjgjgkllujkjjjkhhjhljjk} for inverse mappings we have
\begin{equation}\label{hjgjgjgkllujkjjjkhhjhljjkjhhjh11}
\vec f'(\vec x',t')=\vec f'\left(A(t)\cdot\vec x+\vec
z(t),t\right)=A(t_0)\cdot\vec f(\vec x,t)+\vec z(t_0).
\end{equation}
In particular,
\begin{equation}\label{hjgjgjgkllujkjjjkhhjh11jhgjjggjgh11}
d_{\vec y'}\vec r'(t',\vec y')=A(t)\cdot d_{\vec y}\vec r(t,\vec
y)\cdot\left\{A(t_0)\right\}^{-1}=A(t)\cdot d_{\vec y}\vec r(t,\vec
y)\cdot A^{T}(t_0),
\end{equation}
and
\begin{equation}\label{hjgjgjgkllujkjjjkhhjhljjkjhhjh11hkkhjhgj11}
d_{\vec x'}\vec f'(\vec x',t')=A(t_0)\cdot d_{\vec x}\vec f(\vec
x,t)\cdot\left\{A(t)\right\}^{-1}=A(t_0)\cdot d_{\vec x}\vec f(\vec
x,t)\cdot A^{T}(t).
\end{equation}
However, by \er{guigikvbvbggjklhjkkgjgGGGG} and the fifth equation
in \er{guigikvbvbGG} we have
\begin{equation}\label{guigikvbvbggjklhjkkgjgGGGGjjhjhhjhooiuui}
\begin{cases}
\vec D'(\vec x',t')=A(t)\cdot \vec D(\vec x,t)\\
\vec B'(\vec x',t')=A(t)\cdot\vec B(\vec x,t)\\
\left(\vec E'(\vec x',t')+\frac{1}{c}\,\vec u'(\vec x',t')\times
\vec B'(\vec x',t')\right)= A(t)\cdot\left(\vec
E\left(x,t\right)+\frac{1}{c}\,\vec u\left(x,t\right)\times \vec
B\left(x,t\right)\right)\\
\left(\vec H'(\vec x',t')-\frac{1}{c}\,\vec u'(\vec x',t')\times
\vec D'(\vec x',t')\right)= A(t)\cdot\left(\vec
H\left(x,t\right)-\frac{1}{c}\,\vec u\left(x,t\right)\times \vec
D\left(x,t\right)\right).
\end{cases}
\end{equation}
In, particular, by inserting \er{jjjjkkkkkhhhhjhjh11},
\er{hjgjgjgkllujkjjjkhhjh11} and
\er{hjgjgjgkllujkjjjkhhjh11jhgjjggjgh11} into
\er{EBDHTrans444knbuihuigGGjj11ghthgtjkjj11} and using
\er{guigikvbvbggjklhjkkgjgGGGGjjhjhhjhooiuui} we deduce:
\begin{equation}\label{EBDHTrans444knbuihuigGGjj11ghthgtjkjj11hhhgm11jjhkhk11}
\begin{cases}
{\vec E^*}'(t',\vec y')=A(t_0)\cdot\vec
E^*(t,\vec y)\\
{\vec B^*}'(t',\vec y')=A(t_0)\cdot\vec B^*(t,\vec y)\\
{\vec D^*}'(t',\vec y')=A(t_0)\cdot\vec
D^*(t,\vec y)\\
{\vec H^*}'(t',\vec y')=A(t_0)\cdot\vec H^*(t,\vec y).
\end{cases}
\end{equation}
Thus inserting
\er{EBDHTrans444knbuihuigGGjj11ghthgtjkjj11hhhgm11jjhkhk11} into
\er{MaxVacFullPPNmmmffffffhhtygghGGGGyuhggghghghffggggjj11jhhhjj11}
and
\er{MaxVacFullPPNmmmffffffhhtygghGGGGyuhggghghghffggggjj11jhhhjj11ihhgjjj11}
gives
\begin{equation}\label{EBDHTrans444knbuihuigGGjj11ghthgtjkjj11hhhgm11jjhkhk11jggj11}
\begin{cases}
\hat {\vec{E^*}'}(\tilde\omega',\vec y')=A(t_0)\cdot\vec {\hat
E^*}(\tilde\omega,\vec y)\\
\hat {\vec{B^*}'}(\tilde\omega',\vec y')=A(t_0)\cdot\vec {\hat
B^*}(\tilde\omega,\vec y)\\
\hat {\vec{D^*}'}(\tilde\omega',\vec y')=A(t_0)\cdot\vec {\hat
D^*}(\tilde\omega,\vec y)\\
\hat {\vec{H^*}'}(\tilde\omega',\vec y')=A(t_0)\cdot\vec {\hat
H^*}(\tilde\omega,\vec y),
\end{cases}
\end{equation}
provided that we have
\begin{equation}\label{hjgjgjgkllujkjjjkhhjhljjkjhhjh11hkkhjhgj11jkjkjkgjg11}
\tilde\omega'=\tilde\omega.
\end{equation}
Then, inserting \er{hjgjgjgkllujkjjjkhhjhljjkjhhjh11} and
\er{hjgjgjgkllujkjjjkhhjhljjkjhhjh11hkkhjhgj11} into
\er{EBDHTrans444knbuihuigGGjj11ghthgtjkjj11hhhgm11jjhkhk11jggj11} we
deduce
\begin{equation}\label{EBDHTrans444knbuihuigGGjj11jgggggjj11hhhjjj11gvghhghgjjjhkhkhjh11}
\begin{cases}
\left\{d_{\vec x'}\vec f'\left(\vec x',t'\right)\right\}^T\cdot \hat
{\vec{E^*}'}\left(\tilde\omega',\vec f'\left(\vec
x',t'\right)\right)=A(t)\cdot\left(\left\{d_{\vec x}\vec f\left(\vec
x,t\right)\right\}^T\cdot\vec {\hat E^*}\left(\tilde\omega,\vec
f\left(\vec
x,t\right)\right)\right),\\
\left\{d_{\vec x'}\vec f'\left(\vec x',t'\right)\right\}^T\cdot \hat
{\vec{B^*}'}\left(\tilde\omega',\vec f'\left(\vec
x',t'\right)\right)= A(t)\cdot\left(\left\{d_{\vec x}\vec
f\left(\vec x,t\right)\right\}^T\cdot\vec {\hat
B^*}\left(\tilde\omega,\vec f\left(\vec x,t\right)\right)\right)\\
\left\{d_{\vec x'}\vec f'\left(\vec x',t'\right)\right\}^T\cdot \hat
{\vec{D^*}'}\left(\tilde\omega',\vec f'\left(\vec
x',t'\right)\right)=A(t)\cdot\left(\left\{d_{\vec x}\vec f\left(\vec
x,t\right)\right\}^T\cdot\vec {\hat D^*}\left(\tilde\omega,\vec
f\left(\vec
x,t\right)\right)\right),\\
\left\{d_{\vec x'}\vec f'\left(\vec x',t'\right)\right\}^T\cdot \hat
{\vec{H^*}'}\left(\tilde\omega',\vec f'\left(\vec
x',t'\right)\right)= A(t)\cdot\left(\left\{d_{\vec x}\vec
f\left(\vec x,t\right)\right\}^T\cdot\vec {\hat
H^*}\left(\tilde\omega,\vec f\left(\vec x,t\right)\right)\right).
\end{cases}
\end{equation}
Thus, inserting
\er{EBDHTrans444knbuihuigGGjj11jgggggjj11hhhjjj11gvghhghgjjjhkhkhjh11}
into \er{EBDHTrans444knbuihuigGGjj11jgggggjj11hhhjjj11gvghhghgjj11}
we obtain
\begin{equation}\label{EBDHTrans444knbuihuigGGjj11jgggggjj11hhhjjj11gvghhghgjj11hkhhkjhjh11}
\begin{cases}
\vec P'\left(\vec x',t'\right)=A(t)\cdot\vec P\left(\vec x,t\right)\\
\vec M'\left(\vec x',t'\right)=A(t)\cdot\vec M\left(\vec
x,t\right)-\frac{1}{c}\,\left(A'(t)\cdot\vec x+\vec
w(t)\right)\times \left(A(t)\cdot\vec P\right),
\end{cases}
\end{equation}
provided that we have
\begin{equation}\label{hjgjgjgkllujkjjjkhhjhljjkjhhjh11hkkhjhgj11jkjkjkgjg11hkhk4hjhjh11}
\tilde\omega'=\tilde\omega,\quad n'\left(\tilde\omega',\vec
x',t'\right)=n\left(\tilde\omega,\vec x,t\right)\quad\text{and}\quad
\kappa'\left(\tilde\omega',\vec
x',t'\right)=\kappa\left(\tilde\omega,\vec x,t\right).
\end{equation}
On the other hand,
\er{EBDHTrans444knbuihuigGGjj11jgggggjj11hhhjjj11gvghhghgjj11hkhhkjhjh11}
is fully consistent with the last two equalities in
\er{guigikvbvbGG} and therefore, the dispersion law, in the forms of
\er{EBDHTrans444knbuihuigGGjj11jgggggjj11hhhjjj11} or equivalently
\er{EBDHTrans444knbuihuigGGjj11jgggggjj11hhhjjj11gvghhghgjj11}, is
invariant under the change of inertial or non-inertial cartesian
coordinate system, provided we have
\er{hjgjgjgkllujkjjjkhhjhljjkjhhjh11hkkhjhgj11jkjkjkgjg11hkhk4hjhjh11}.

Next we will prove that the right hand side of
\er{EBDHTrans444knbuihuigGGjj11jgggggjj11hhhjjj11} or
\er{EBDHTrans444knbuihuigGGjj11jgggggjj11hhhjjj11gvghhghgjj11} is
independent of the initial instant of time $t_0$ in
\er{hjgghhghhffhkkjj11}. Indeed let $t_1\in\mathbb{R}$ be another
choice of initial instant of time, let $\vec r(t,\vec y)$ be given
by \er{hjgghhghhffhkkjj11}:
\begin{equation}\label{hjgghhghhffhkkjj11jj11}
\begin{cases}
\frac{\partial\vec r}{\partial t}(t,\vec y)=\vec u\left(\vec
r(t,\vec y),t\right)\\
\vec r(t_0,\vec y)=\vec y,
\end{cases}
\end{equation}
and let $\vec r_1(t,\vec y)$ be given by the following:
\begin{equation}\label{hjgghhghhffhkkjj11jj11jgjgjg11}
\begin{cases}
\frac{\partial\vec r_1}{\partial t}(t,\vec y)=\vec u\left(\vec
r_1(t,\vec y),t\right)\\
\vec r_1(t_1,\vec y)=\vec y.
\end{cases}
\end{equation}
Then considering $\vec r_2(t,\vec y)$ defined by
\begin{equation}\label{hjgghhghhffhkkjj11jj11gghgh11khjhjh11}
\vec r_2(t,\vec y)=\vec r\left(t,\vec h(\vec
y)\right)\quad\text{where}\quad \vec h(\vec y):=\vec r_1(t_0,\vec
y),
\end{equation}
gives
\begin{equation}\label{hjgghhghhffhkkjj11jj11jhjhjhjh11}
\begin{cases}
\frac{\partial\vec r_2}{\partial t}(t,\vec y)=\vec u\left(\vec
r_2(t,\vec y),t\right)\\
\vec r_2(t_0,\vec y)=\vec r\left(t_0,\vec r_1(t_0,\vec
y)\right)=\vec r_1(t_0,\vec y).
\end{cases}
\end{equation}
Thus, by \er{hjgghhghhffhkkjj11jj11jhjhjhjh11}, using the uniqueness
of solution to the ordinary differential equation we deduce
\begin{equation}\label{hjgghhghhffhkkjj11hgjggjgjf11}
\vec r_1(t,\vec y)=\vec r_2(t,\vec y)=\vec r\left(t,\vec h(\vec
y)\right)\quad\quad\forall \,t,\vec y.
\end{equation}
Thus considering,
\begin{equation}\label{EBDHTrans444knbuihuigGGjj11ghthgtjkjj11jhjgh11}
\begin{cases}
\vec D^*_1(t,\vec y):=\left\{d_{\vec y}\vec r_1\left(t,\vec
y\right)\right\}^T\cdot\vec D\left(\vec r_1(t,\vec y),t\right)\\
\vec B^*_1(t,\vec y):=\left\{d_{\vec y}\vec r_1\left(t,\vec
y\right)\right\}^T\cdot\vec B\left(\vec r_1(t,\vec y),t\right)\\
\vec E^*_1(t,\vec y):=\left\{d_{\vec y}\vec r_1\left(t,\vec
y\right)\right\}^T\cdot\left(\vec E\left(\vec r_1(t,\vec
y),t\right)+\frac{1}{c}\,\vec u\left(\vec r_1(t,\vec
y),t\right)\times \vec B\left(\vec r_1(t,\vec
y),t\right)\right)\\
\vec H^*_1(t,\vec y):=\left\{d_{\vec y}\vec r_1\left(t,\vec
y\right)\right\}^T\cdot\left(\vec H\left(\vec r_1(t,\vec
y),t\right)-\frac{1}{c}\,\vec u\left(\vec r_1(t,\vec
y),t\right)\times \vec D\left(\vec r_1(t,\vec y),t\right)\right),
\end{cases}
\end{equation}
where $\vec r_1(t,\vec y)$ is given by
\er{hjgghhghhffhkkjj11jj11jgjgjg11}, by
\er{hjgghhghhffhkkjj11hgjggjgjf11} we deduce
\begin{multline}\label{EBDHTrans444knbuihuigGGjj11ghthgtjkjj11jhjgh11jj11}
\vec E^*_1(t,\vec y)=\left\{d_{\vec y}\vec h\left(\vec
y\right)\right\}^T\cdot\vec E^*\left(t,\vec h(\vec
y)\right)\,,\quad\quad\vec B^*_1(t,\vec y)=\left\{d_{\vec y}\vec
h\left(\vec y\right)\right\}^T\cdot\vec B^*\left(t,\vec
h(\vec y)\right)\,,\\
\vec D^*_1(t,\vec y)=\left\{d_{\vec y}\vec h\left(\vec
y\right)\right\}^T\cdot\vec D^*\left(t,\vec h(\vec
y)\right)\quad\text{and}\quad\vec H^*_1(t,\vec y)=\left\{d_{\vec
y}\vec h\left(\vec y\right)\right\}^T\cdot\vec H^*\left(t,\vec
h(\vec y)\right).
\end{multline}
Then defining
\begin{multline}\label{MaxVacFullPPNmmmffffffhhtygghGGGGyuhggghghghffggggjj11jhhhjj11jjhj11bbbvvvvbbc11hkhk11}
{\vec {\hat E^*}}_1(\tilde\omega,\vec
y):=\frac{1}{2\pi}\int_{-\infty}^{+\infty} \vec E^*_1(t,\vec
y)e^{-i\tilde \omega t}dt\,,\quad\quad {\vec {\hat
D^*}}_1(\tilde\omega,\vec y):=\frac{1}{2\pi}\int_{-\infty}^{+\infty}
\vec D^*_1(t,\vec y)e^{-i\tilde \omega t}dt\,,\\
{\vec {\hat B^*}}_1(\tilde\omega,\vec
y):=\frac{1}{2\pi}\int_{-\infty}^{+\infty} \vec B^*_1(t,\vec
y)e^{-i\tilde \omega t}dt\quad\text{and}\quad {\vec {\hat
H^*}}_1(\tilde\omega,\vec y):=\frac{1}{2\pi}\int_{-\infty}^{+\infty}
\vec H^*_1(t,\vec y)e^{-i\tilde \omega t}dt,
\end{multline}
we deduce
\begin{multline}\label{MaxVacFullPPNmmmffffffhhtygghGGGGyuhggghghghffggggjj11jhhhjj11jjhj11bbbvvvvbbc11jkkhkfhfh11}
{\vec {\hat E^*}}_1(\tilde\omega,\vec y)=\left\{d_{\vec y}\vec
h\left(\vec y\right)\right\}^T\cdot{\vec {\hat
E^*}}\left(\tilde\omega,\vec h(\vec y)\right)\,,\quad\quad {\vec
{\hat B^*}}_1(\tilde\omega,\vec y)=\left\{d_{\vec y}\vec h\left(\vec
y\right)\right\}^T\cdot{\vec {\hat
B^*}}\left(\tilde\omega,\vec h(\vec y)\right)\,,\\
{\vec {\hat D^*}}_1(\tilde\omega,\vec y)=\left\{d_{\vec y}\vec
h\left(\vec y\right)\right\}^T\cdot{\vec {\hat
D^*}}\left(\tilde\omega,\vec h(\vec y)\right)\quad\text{and}\quad
{\vec {\hat H^*}}_1(\tilde\omega,\vec y)=\left\{d_{\vec y}\vec
h\left(\vec y\right)\right\}^T\cdot{\vec {\hat
H^*}}\left(\tilde\omega,\vec h(\vec y)\right).
\end{multline}
Thus, we write the analog of
\er{EBDHTrans444knbuihuigGGjj11jgggggjj11hhhjjj11} as:
\begin{multline}\label{EBDHTrans444knbuihuigGGjj11jgggggjj11hhhjjj11jhjgj11jkhkjh11}
\vec P\left(\vec r_1(t,\vec y),t\right)=\vec P\left(\vec
r\left(t,\vec h(\vec y)\right),t\right)=\\
Re\left\{2\int_{0}^{+\infty} \frac{n^2\left(\tilde\omega,\vec
r\left(t,\vec h(\vec
y)\right),t\right)-1}{4\pi}\,\left(\left\{d_{\vec y}\vec
r\left(t,\vec h(\vec y)\right)\cdot d_{\vec y}\vec h(\vec
y)\right\}^T\right)^{-1}\cdot\left\{d_{\vec y}\vec h\left(\vec
y\right)\right\}^T\cdot\vec {\hat E^*}\left(\tilde\omega,\vec
h(y)\right)e^{i\tilde \omega
t}d\tilde\omega\right\}\\-Re\left\{2\int_{0}^{+\infty}
\kappa\left(\tilde\omega,\vec r\left(t,\vec h(\vec
y)\right),t\right)\left(\left\{d_{\vec y}\vec r\left(t,\vec h(\vec
y)\right)\cdot d_{\vec y}\vec h(\vec
y)\right\}^T\right)^{-1}\cdot\left\{d_{\vec y}\vec h\left(\vec
y\right)\right\}^T\cdot\vec {\hat D^*}\left(\tilde\omega,\vec
h(y)\right)e^{i\tilde \omega t}d\tilde\omega\right\}
\\= Re\left\{2\int_{0}^{+\infty}
\frac{n^2\left(\tilde\omega,\vec r_1(t,\vec
y),t\right)-1}{4\pi}\,\left(\left\{d_{\vec y}\vec r_1\left(t,\vec
y\right)\right\}^T\right)^{-1}\cdot{\vec {\hat
E^*}}_1(\tilde\omega,\vec y)e^{i\tilde \omega
t}d\tilde\omega\right\}\\-Re\left\{2\int_{0}^{+\infty}
\kappa\left(\tilde\omega,\vec r_1(t,\vec
y),t\right)\left(\left\{d_{\vec y}\vec r_1\left(t,\vec
y\right)\right\}^T\right)^{-1}\cdot{\vec {\hat
D^*}}_1(\tilde\omega,\vec y)e^{i\tilde \omega
t}d\tilde\omega\right\},
\end{multline}
and
\begin{multline}\label{EBDHTrans444knbuihuigGGjj11jgggggjj11hhhjjj11jhjgj11jhgjjggh11gjgj11}
\vec M\left(\vec r_1(t,\vec y),t\right)+\frac{1}{c}\vec u\left(\vec
r_1(t,\vec y),t\right)\times \vec P\left(\vec r_1(t,\vec
y),t\right)=\\
\vec M\left(\vec r\left(t,\vec h(\vec
y)\right),t\right)+\frac{1}{c}\vec u\left(\vec r\left(t,\vec h(\vec
y)\right),t\right)\times \vec P\left(\vec r\left(t,\vec h(\vec
y)\right),t\right)=\\ Re\left\{2\int_{0}^{+\infty}
\kappa\left(\tilde\omega,\vec r\left(t,\vec h(\vec
y)\right),t\right)\left(\left\{d_{\vec y}\vec r\left(t,\vec h(\vec
y)\right)\cdot d_{\vec y}\vec h(\vec
y)\right\}^T\right)^{-1}\cdot\left\{d_{\vec y}\vec h\left(\vec
y\right)\right\}^T\cdot\vec {\hat H^*}\left(\tilde\omega,\vec
h(y)\right)e^{i\tilde \omega t}d\tilde\omega\right\}
\\= Re\left\{2\int_{0}^{+\infty}
\kappa\left(\tilde\omega,\vec r_1(t,\vec
y),t\right)\left(\left\{d_{\vec y}\vec r_1\left(t,\vec
y\right)\right\}^T\right)^{-1}\cdot{\vec {\hat
H^*}}_1(\tilde\omega,\vec y)e^{i\tilde \omega
t}d\tilde\omega\right\}.
\end{multline}
So we proved that the right hand side of
\er{EBDHTrans444knbuihuigGGjj11jgggggjj11hhhjjj11} or
\er{EBDHTrans444knbuihuigGGjj11jgggggjj11hhhjjj11gvghhghgjj11} is
independent of the initial instant of time $t_0$.

Finally, we can write the following possible alternative to the
dispersion law in the form of
\er{EBDHTrans444knbuihuigGGjj11ghthgtjkjj11} and
\er{EBDHTrans444knbuihuigGGjj11jgggggjj11hhhjjj11} or
\er{EBDHTrans444knbuihuigGGjj11jgggggjj11hhhjjj11gvghhghgjj11}. We
could consider,
\begin{equation}\label{EBDHTrans444knbuihuigGGjj11ghthgtjkjj11jj11}
\begin{cases}
\vec E^*(t,\vec y):=\left\{d_{\vec y}\vec r\left(t,\vec
y\right)\right\}^{-1}\cdot\left(\vec E\left(\vec r(t,\vec
y),t\right)+\frac{1}{c}\,\vec u\left(\vec r(t,\vec y),t\right)\times
\vec B\left(\vec r(t,\vec y),t\right)\right)\quad\quad\forall
t\in\mathbb{R},\;\;\forall\vec
y\in\Omega,\\
\vec B^*(t,\vec y):=\left\{d_{\vec y}\vec r\left(t,\vec
y\right)\right\}^{-1}\cdot\vec B\left(\vec r(t,\vec
y),t\right)\quad\quad\forall t\in\mathbb{R},\;\;\forall\vec
y\in\Omega,\\
\vec D^*(t,\vec y):=\left\{d_{\vec y}\vec r\left(t,\vec
y\right)\right\}^{-1}\cdot\vec D\left(\vec r(t,\vec
y),t\right)\quad\quad\forall t\in\mathbb{R},\;\;\forall\vec
y\in\Omega,\\
\vec H^*(t,\vec y):=\left\{d_{\vec y}\vec r\left(t,\vec
y\right)\right\}^{-1}\cdot\left(\vec H\left(\vec r(t,\vec
y),t\right)-\frac{1}{c}\,\vec u\left(\vec r(t,\vec y),t\right)\times
\vec D\left(\vec r(t,\vec y),t\right)\right)\quad\quad\forall
t\in\mathbb{R},\;\;\forall\vec y\in\Omega,
\end{cases}
\end{equation}
where $\vec r(t,\vec y)$ is given by \er{hjgghhghhffhkkjj11}.
Then, since $\vec E^*$, $\vec D^*$, $\vec B^*$ and $\vec H^*$ are
real vectors, we could write them as a Fourier's Transform on the
time variable:
\begin{align}\label{MaxVacFullPPNmmmffffffhhtygghGGGGyuhggghghghffggggjj11jhhhjj11jj11}
\vec E^*(t,\vec y)=Re\left\{2\int_{0}^{+\infty}  {\vec {\hat
E^*}(\tilde\omega,\vec y)}e^{i\tilde \omega
t}d\tilde\omega\right\}\quad\text{where}\quad {\vec {\hat
E^*}(\tilde\omega,\vec y)}:=\frac{1}{2\pi}\int_{-\infty}^{+\infty}
\vec E^*(t,\vec y)e^{-i\tilde \omega t}dt\,,\\
\label{MaxVacFullPPNmmmffffffhhtygghGGGGyuhggghghghffggggjj11jhhhjj11ihhgjjj11bhghhhg11}
\vec B^*(t,\vec y)=Re\left\{2\int_{0}^{+\infty}  {\vec {\hat
B^*}(\tilde\omega,\vec y)}e^{i\tilde \omega
t}d\tilde\omega\right\}\quad\text{where}\quad {\vec {\hat
B^*}(\tilde\omega,\vec y)}:=\frac{1}{2\pi}\int_{-\infty}^{+\infty}
\vec B^*(t,\vec y)e^{-i\tilde \omega t}dt \,,\\
\label{MaxVacFullPPNmmmffffffhhtygghGGGGyuhggghghghffggggjj11jhhhjj11jj11D}
\vec D^*(t,\vec y)=Re\left\{2\int_{0}^{+\infty}  {\vec {\hat
D^*}(\tilde\omega,\vec y)}e^{i\tilde \omega
t}d\tilde\omega\right\}\quad\text{where}\quad {\vec {\hat
D^*}(\tilde\omega,\vec y)}:=\frac{1}{2\pi}\int_{-\infty}^{+\infty}
\vec D^*(t,\vec y)e^{-i\tilde \omega t}dt\,,\\
\label{MaxVacFullPPNmmmffffffhhtygghGGGGyuhggghghghffggggjj11jhhhjj11ihhgjjj11bhghhhg11H}
\vec H^*(t,\vec y)=Re\left\{2\int_{0}^{+\infty}  {\vec {\hat
H^*}(\tilde\omega,\vec y)}e^{i\tilde \omega
t}d\tilde\omega\right\}\quad\text{where}\quad {\vec {\hat
H^*}(\tilde\omega,\vec y)}:=\frac{1}{2\pi}\int_{-\infty}^{+\infty}
\vec H^*(t,\vec y)e^{-i\tilde \omega t}dt\,.
\end{align}
Then, alternatively to \er{EBDHTrans444knbuihuigGGjj11ghthgtjkjj11}
and \er{EBDHTrans444knbuihuigGGjj11jgggggjj11hhhjjj11} we could
consider
\begin{multline}\label{EBDHTrans444knbuihuigGGjj11jgggggjj11hhhjjj11jj11}
\vec P\left(\vec r(t,\vec y),t\right)= Re\left\{2\int_{0}^{+\infty}
{\frac{n^2\left(\tilde\omega,\vec r(t,\vec
y),t\right)-1}{4\pi}\,\left(d_{\vec y}\vec r\left(t,\vec
y\right)\cdot\vec {\hat E^*}\left(\omega,\vec
 x\right)\right)}e^{i\omega t}d\omega\right\}\\-Re\left\{2\int_{0}^{+\infty}
{\kappa\left(\tilde\omega,\vec r(t,\vec y),t\right)\left(d_{\vec
y}\vec r\left(t,\vec y\right)\cdot\vec {\hat D^*}\left(\omega,\vec
x\right)\right)}e^{i\omega t}d\omega\right\}\\
\quad\text{and}\quad\quad\quad \vec M\left(\vec r(t,\vec
y),t\right)+\frac{1}{c}\vec u\left(\vec
r(t,\vec y),t\right)\times \vec P\left(\vec r(t,\vec y),t\right)=\\
Re\left\{2\int_{0}^{+\infty} {\kappa\left(\tilde\omega,\vec r(t,\vec
y),t\right)\left(d_{\vec y}\vec r\left(t,\vec y\right)\cdot\vec
{\hat H^*}(\tilde\omega,\vec y)\right)}e^{i\tilde \omega
t}d\tilde\omega\right\},
\end{multline}
i.e.
\begin{multline}\label{EBDHTrans444knbuihuigGGjj11jgggggjj11hhhjjj11gvghhghgjj11jj11}
\vec P\left(\vec x,t\right)=Re\left\{2\int_{0}^{+\infty}
{\frac{n^2\left(\tilde\omega,\vec
x,t\right)-1}{4\pi}\,\left(\left\{d_{\vec x}\vec f\left(\vec
x,t\right)\right\}^{-1}\cdot\vec {\hat E^*}\left(\tilde\omega,\vec
f\left(\vec x,t\right)\right)\right)}e^{i\tilde \omega
t}d\tilde\omega\right\}\\-Re\left\{2\int_{0}^{+\infty}
{\kappa\left(\tilde\omega,\vec x,t\right)\left(\left\{d_{\vec x}\vec
f\left(\vec x,t\right)\right\}^{-1}\cdot\vec {\hat
D^*}\left(\tilde\omega,\vec f\left(\vec
x,t\right)\right)\right)}e^{i\tilde \omega t}d\tilde\omega\right\}\quad\quad\text{and}\\
\vec M\left(\vec x,t\right)+\frac{1}{c}\vec u(\vec x,t)\times \vec
P(\vec x,t)=Re\left\{2\int_{0}^{+\infty}
{\kappa\left(\tilde\omega,\vec x,t\right)\left(\left\{d_{\vec x}\vec
f\left(\vec x,t\right)\right\}^{-1}\cdot\vec {\hat
H^*}\left(\tilde\omega,\vec f\left(\vec
x,t\right)\right)\right)}e^{i\tilde \omega t}d\tilde\omega\right\}.
\end{multline}
Then, exactly in the same way we already proved for
\er{EBDHTrans444knbuihuigGGjj11ghthgtjkjj11} and
\er{EBDHTrans444knbuihuigGGjj11jgggggjj11hhhjjj11} or
\er{EBDHTrans444knbuihuigGGjj11jgggggjj11hhhjjj11gvghhghgjj11} we
can also prove that the above alternative
\er{EBDHTrans444knbuihuigGGjj11ghthgtjkjj11jj11} and
\er{EBDHTrans444knbuihuigGGjj11jgggggjj11hhhjjj11jj11} or
\er{EBDHTrans444knbuihuigGGjj11jgggggjj11hhhjjj11gvghhghgjj11jj11}
is invariant under the change of inertial or non-inertial cartesian
coordinate system. Furthermore, as before, the above alternative is
equivalent to
\er{EBDHTrans444knbuihuigGGjj11jgggggjj11hhhjjj11gvghhghgjj11jgjggj11hghffh11}
in the particular case $\vec u\equiv 0$. Moreover, in the case of an
arbitrary moving transparent medium without dispersion the
alternative law is equivalent to \er{EBDHTrans444knbuihuigGG}.
Finally, as before, \er{EBDHTrans444knbuihuigGGjj11ghthgtjkjj11jj11}
and \er{EBDHTrans444knbuihuigGGjj11jgggggjj11hhhjjj11jj11} or
\er{EBDHTrans444knbuihuigGGjj11jgggggjj11hhhjjj11gvghhghgjj11jj11}
is also independent of the initial instant of time $t_0$.

We are unable to see any clear advantage of the dispersion law in
the form of \er{EBDHTrans444knbuihuigGGjj11ghthgtjkjj11} and
\er{EBDHTrans444knbuihuigGGjj11jgggggjj11hhhjjj11} or
\er{EBDHTrans444knbuihuigGGjj11jgggggjj11hhhjjj11gvghhghgjj11} with
respect to the law \er{EBDHTrans444knbuihuigGGjj11ghthgtjkjj11jj11}
and \er{EBDHTrans444knbuihuigGGjj11jgggggjj11hhhjjj11jj11} or
\er{EBDHTrans444knbuihuigGGjj11jgggggjj11hhhjjj11gvghhghgjj11jj11}.
We just note that the above two forms of the dispersion law coincide
in the case where our medium moves as a \underline{rigid} body.

\subsubsection{Optical dispersion in anisotropic moving
mediums}\label{hugygyggygygyfgyhkjhjhgj} The generalization of
\er{EBDHTrans444knbuihuigGGjj11jgggggjj11hhhjjj11gvghhghgjj11} to
anisotropic mediums is the following
\begin{multline}\label{EBDHTrans444knbuihuigGGjj11jgggggjj11hhhjjj11gvghhghgjj11hjhjhg11}
\vec P\left(\vec x,t\right)=Re\left\{2\int_{0}^{+\infty}
{\Gamma\left(\tilde\omega,\vec x,t\right)\cdot\left(\left\{d_{\vec
x}\vec f\left(\vec x,t\right)\right\}^T\cdot\vec {\hat
E^*}\left(\tilde\omega,\vec f\left(\vec
x,t\right)\right)\right)}e^{i\tilde \omega
t}d\tilde\omega\right\}\\-Re\left\{2\int_{0}^{+\infty}
{\Upsilon\left(\tilde\omega,\vec x,t\right)\cdot\left(\left\{d_{\vec
x}\vec f\left(\vec x,t\right)\right\}^T\cdot\vec {\hat
D^*}\left(\tilde\omega,\vec f\left(\vec
x,t\right)\right)\right)}e^{i\tilde \omega t}d\tilde\omega\right\}\quad\quad\text{and}\\
\vec M\left(\vec x,t\right)+\frac{1}{c}\vec u\left(\vec
x,t\right)\times\vec P\left(\vec
x,t\right)=Re\left\{2\int_{0}^{+\infty}
{\Upsilon\left(\tilde\omega,\vec x,t\right)\cdot\left(\left\{d_{\vec
x}\vec f\left(\vec x,t\right)\right\}^T\cdot\vec {\hat
H^*}\left(\tilde\omega,\vec f\left(\vec
x,t\right)\right)\right)}e^{i\tilde \omega t}d\tilde\omega\right\}.
\end{multline}
Here  $\vec E^*$, $\vec D^*$, $\vec B^*$ and $\vec H^*$ are given by
\er{EBDHTrans444knbuihuigGGjj11ghthgtjkjj11} with $\vec {\hat E^*}$,
$\vec {\hat D^*}$, $\vec {\hat B^*}$ and $\vec {\hat H^*}$ being
their Fourier's Transforms on the time variable and
$\Gamma:=\Gamma\left(\tilde\omega,\vec
x,t\right)\in\mathbb{C}^{3\times 3}$ and
$\Upsilon:=\Upsilon\left(\tilde\omega,\vec
x,t\right)\in\mathbb{C}^{3\times 3}$ are complex matrices describing
the dispersion. Then,
\er{EBDHTrans444knbuihuigGGjj11jgggggjj11hhhjjj11gvghhghgjj11hjhjhg11}
is invariant under the change of inertial or non-inertial coordinate
system, provided that under \er{jjjjkkkk11} we have
\begin{equation}\label{hjgjgjgkllujkjjjkhhjhljjkjhhjh11hkkhjhgj11jkjkjkgjg11hkhk4hjhjh11mnmbmbmbjh11}
\tilde\omega'=\tilde\omega,\quad \Gamma'\left(\tilde\omega',\vec
x',t'\right)=A(t)\cdot\Gamma\left(\tilde\omega,\vec x,t\right)\cdot
A^T(t)\quad\text{and}\quad \Upsilon'\left(\tilde\omega',\vec
x',t'\right)=A(t)\cdot\Upsilon\left(\tilde\omega,\vec
x,t\right)\cdot A^T(t).
\end{equation}
Next, in the case of an anisotropic transparent medium without
dispersion i.e. in the case where we assume that
$\Gamma:=\Gamma\left(\vec x,t\right)$ and
$\Upsilon:=\Upsilon\left(\vec x,t\right)$ are independent of the
argument $\tilde\omega$ and moreover we assume them to be
\underline{real} matrices,
\er{EBDHTrans444knbuihuigGGjj11jgggggjj11hhhjjj11gvghhghgjj11hjhjhg11}
coincides with \er{EBDHTrans444knbuihuigGGaa}. Finally, as before,
the right hand side of
\er{EBDHTrans444knbuihuigGGjj11jgggggjj11hhhjjj11gvghhghgjj11hjhjhg11}
is independent of the initial instant of time $t_0$ in
\er{hjgghhghhffhkkjj11}.

On the other hand, the generalization of
\er{EBDHTrans444knbuihuigGGjj11jgggggjj11hhhjjj11gvghhghgjj11jj11}
to anisotropic mediums is the following
\begin{multline}\label{EBDHTrans444knbuihuigGGjj11jgggggjj11hhhjjj11gvghhghgjj11jj11khghgjg11}
\vec P\left(\vec x,t\right)=Re\left\{2\int_{0}^{+\infty}
{\Gamma\left(\tilde\omega,\vec x,t\right)\cdot\left(\left\{d_{\vec
x}\vec f\left(\vec x,t\right)\right\}^{-1}\cdot\vec {\hat
E^*}\left(\tilde\omega,\vec f\left(\vec
x,t\right)\right)\right)}e^{i\tilde \omega t}d\tilde\omega\right\}-\\
Re\left\{2\int_{0}^{+\infty} {\Upsilon\left(\tilde\omega,\vec
x,t\right)\cdot\left(\left\{d_{\vec x}\vec f\left(\vec
x,t\right)\right\}^{-1}\cdot\vec {\hat D^*}\left(\tilde\omega,\vec
f\left(\vec
x,t\right)\right)\right)}e^{i\tilde \omega t}d\tilde\omega\right\}\quad\quad\text{and}\\
\vec M\left(\vec x,t\right)+\frac{1}{c}\vec u\left(\vec
x,t\right)\times\vec P\left(\vec
x,t\right)=Re\left\{2\int_{0}^{+\infty}
{\Upsilon\left(\tilde\omega,\vec x,t\right)\cdot\left(\left\{d_{\vec
x}\vec f\left(\vec x,t\right)\right\}^{-1}\cdot\vec {\hat
H^*}\left(\tilde\omega,\vec f\left(\vec
x,t\right)\right)\right)}e^{i\tilde \omega t}d\tilde\omega\right\}.
\end{multline}
Now  $\vec E^*$ and $\vec B^*$ are given by
\er{EBDHTrans444knbuihuigGGjj11ghthgtjkjj11jj11} with $\vec {\hat
E^*}$ and $\vec {\hat B^*}$ being their Fourier's Transforms on the
time variable and $\Gamma:=\Gamma\left(\tilde\omega,\vec
x,t\right)\in\mathbb{C}^{3\times 3}$ and
$\Upsilon:=\Upsilon\left(\tilde\omega,\vec
x,t\right)\in\mathbb{C}^{3\times 3}$ are, as before, complex
matrices describing the dispersion. Then again,
\er{EBDHTrans444knbuihuigGGjj11jgggggjj11hhhjjj11gvghhghgjj11jj11khghgjg11}
is invariant under the change of inertial or non-inertial coordinate
system, provided that under \er{jjjjkkkk11} we have
\er{hjgjgjgkllujkjjjkhhjhljjkjhhjh11hkkhjhgj11jkjkjkgjg11hkhk4hjhjh11mnmbmbmbjh11}.
Moreover, in the case of an anisotropic transparent medium without
dispersion i.e. in the case where we assume that
$\Gamma:=\Gamma\left(\vec x,t\right)$ and
$\Upsilon:=\Upsilon\left(\vec x,t\right)$ are independent of the
argument $\tilde\omega$ and moreover we assume them to be
\underline{real} matrices,
\er{EBDHTrans444knbuihuigGGjj11jgggggjj11hhhjjj11gvghhghgjj11jj11khghgjg11}
again coincides with \er{EBDHTrans444knbuihuigGGaa}. Finally, as
before, the right hand side of
\er{EBDHTrans444knbuihuigGGjj11jgggggjj11hhhjjj11gvghhghgjj11jj11khghgjg11}
is independent of the initial instant of time $t_0$ in
\er{hjgghhghhffhkkjj11}.

\section{Some further consequences of Maxwell equations}\label{CM}
\subsection{General case}\label{gcCM}
Again consider the system of Maxwell equations in the vacuum or in a
medium, without dispersion, having the form
\er{MaxMedFullGGffykkhjhhzz}:
\begin{equation}\label{MaxVacFullPPNffGG}
\begin{cases}
curl_{\vec x} \vec H=\frac{4\pi}{c}\vec j+
\frac{1}{c}\frac{\partial \vec D}{\partial t},\\
div_{\vec x} \vec D=4\pi\rho,\\
curl_{\vec x} \vec E+\frac{1}{c}\frac{\partial \vec B}{\partial t}=0,\\
div_{\vec x} \vec B=0\\
\vec E=\gamma_0\vec D-\frac{1}{c}\,\vec {\tilde u}\times \vec B\\
\vec H=\kappa_0\vec B+\frac{1}{c}\,\vec {\tilde u}\times \vec D,\\
\vec {\tilde u}=\left(\gamma_0\kappa_0\vec
v+(1-\gamma_0\kappa_0)\vec u\right),
\end{cases}
\end{equation}
%
%
%
%
%
%
where $\gamma_0\neq 0$ and $\kappa_0\neq 0$ are material
coefficients, $\vec v$ is the vectorial gravitational potential
$\vec u$ is the medium velocity and $\vec {\tilde
u}=\left(\gamma_0\kappa_0\vec v+(1-\gamma_0\kappa_0)\vec u\right)$
is the speed-like vector field. Remind that in the case of the
vacuum we have $\gamma_0=\kappa_0=1$, $\vec {\tilde u}=\vec v$ and
equations \er{MaxVacFullPPNffGG} are precise (in the frames of our
model). Otherwise, in the case $\kappa_0\gamma_0\neq 1$ equations
\er{MaxVacFullPPNffGG} are just an approximation that is good enough
for the case:
\begin{equation}\label{OprdddsimGGffyhjyhhtygrffgfzzjjj}
\frac{|\kappa_0||1-\kappa_0\gamma_0|\cdot|\vec u-\vec v|^2}{c^2}\ll
1.
\end{equation}
Next again by the third and the fourth equations in
\er{MaxVacFullPPNffGG} we can write
\begin{equation}\label{MaxVacFull1bjkgjhjhgjgjgkjfhjfdghghligioiuittrPPNgg}
\begin{cases}
\vec B\equiv curl_{\vec x} \vec A,\\
\vec E\equiv-\nabla_{\vec x}\Psi-\frac{1}{c}\frac{\partial\vec
A}{\partial t},
\end{cases}
\end{equation}
where $\Psi$ and $\vec A$ are the usual scalar and the vectorial
electromagnetic potentials. Then by
\er{MaxVacFull1bjkgjhjhgjgjgkjfhjfdghghligioiuittrPPNgg} and
\er{MaxVacFullPPNffGG} we have
\begin{equation}\label{vhfffngghPPN333yuyuGG}
\begin{cases}
\vec B= curl_{\vec x} \vec A\\
\vec E=-\nabla_{\vec x}\Psi-\frac{1}{c}\frac{\partial\vec
A}{\partial t}\\
 \vec D=-\frac{1}{\gamma_0}\nabla_{\vec
x}\Psi-\frac{1}{\gamma_0 c}\frac{\partial\vec A}{\partial t}+\frac{1}{c\gamma_0}\vec {\tilde u}\times curl_{\vec x}\vec A\\
\vec H= \kappa_0 \,curl_{\vec x} \vec A+\frac{1}{c}\,\vec {\tilde
u}\times\left(-\frac{1}{\gamma_0}\nabla_{\vec
x}\Psi-\frac{1}{\gamma_0 c}\frac{\partial\vec A}{\partial
t}+\frac{1}{\gamma_0 c}\vec {\tilde u}\times curl_{\vec x}\vec
A\right).
\end{cases}
\end{equation}
Next we remind the definition of the proper scalar electromagnetic
potential:
\begin{equation}\label{vhfffngghhjghhgPPNghghghutghffugghjhjkjjklggkkk}
\Psi_0:=\Psi-\frac{1}{c}\vec A\cdot\vec v,
\end{equation}
and remind also that $\vec A$ is a proper vector field and $\Psi_0$
is a proper scalar field. Then in the case of the medium we also
define an additional scalar electromagnetic potential:
\begin{equation}\label{vhfffngghhjghhgPPNghghghutghffugghjhjkjjklgg}
\Psi_1:=\Psi-\frac{1}{c}\vec A\cdot\vec {\tilde u}.
\end{equation}
Then, since $\vec A$ is a proper vector field, we deduce that
$\Psi_1$ is also a proper scalar field. Moreover, in the case of the
vacuum or more generally in the case where $\gamma_0\kappa_0\approx
1$ we have $\Psi_1=\Psi_0$. Thus by
\er{vhfffngghhjghhgPPNghghghutghffugghjhjkjjklgg} we rewrite
\er{vhfffngghPPN333yuyuGG} as:
\begin{equation}\label{vhfffngghPPNffGG}
\begin{cases}
\vec B= curl_{\vec x} \vec A\\
\vec E=-\nabla_{\vec x}\Psi_1-\frac{1}{c}\frac{\partial\vec
A}{\partial t}-\frac{1}{c}\nabla_{\vec x}\left(\vec A\cdot\vec
{\tilde u}\right)
\\
 \vec D=-\frac{1}{\gamma_0}\nabla_{\vec
x}\Psi_1-\frac{1}{\gamma_0 c}\left(\frac{\partial\vec A}{\partial
t}-\vec {\tilde u}\times curl_{\vec x}\vec A+\nabla_{\vec
x}\left(\vec A\cdot\vec {\tilde u}\right)\right)
\\
\vec H= \kappa_0\,curl_{\vec x} \vec A-\frac{1}{c}\,\vec {\tilde
u}\times
\left(\frac{1}{\gamma_0}\nabla_{\vec x}\Psi_1+\frac{1}{\gamma_0
c}\left(\frac{\partial\vec A}{\partial t}-\vec {\tilde u}\times
curl_{\vec x}\vec A+\nabla_{\vec x}\left(\vec A\cdot\vec {\tilde
u}\right)\right)\right).
\end{cases}
\end{equation}
Using Proposition \ref{yghgjtgyrtrt} we rewrite the third equation
in \er{vhfffngghPPNffGG} as
\begin{equation}\label{vhfffngghPPNffGG1}
\vec D=-\frac{1}{\gamma_0}\nabla_{\vec x}\Psi_1-\frac{1}{\gamma_0
c}\left(\frac{\partial\vec A}{\partial t}-curl_{\vec x}\left(\vec
{\tilde u}\times\vec A\right)+\left(div_{\vec x}\vec A\right)\vec
{\tilde u}+\left(d_{\vec x}\vec {\tilde u}+\left\{d_{\vec x}\vec
{\tilde u}\right\}^T\right)\cdot\vec A-\left(div_{\vec x}\vec
{\tilde u}\right)\vec A\right).
\end{equation}
Then by \er{vhfffngghPPNffGG1}, \er{vhfffngghPPNffGG} and
\er{MaxVacFullPPNffGG} we have
\begin{equation}\label{MaxVacFullPPNmmmffGG111}
\Div_{\vec x}\left\{\frac{1}{\gamma_0 c}\left(\frac{\partial\vec
A}{\partial t}-\vec {\tilde u}\times curl_{\vec x}\vec
A+\nabla_{\vec x}\left(\vec A\cdot\vec {\tilde
u}\right)\right)\right\}+\Div_{\vec
x}\left\{\frac{1}{\gamma_0}\,\nabla_{\vec x}\Psi_1\right\}=-4\pi\rho
\end{equation}
and
\begin{multline}\label{MaxVacFullPPNnnnffGG1111}
curl_{\vec x} \left\{\kappa_0\,curl_{\vec x} \vec
A-\frac{1}{\gamma_0 c}\,\vec {\tilde u}\times
\left(\nabla_{\vec x}\Psi_1+\frac{1}{c}\left(\frac{\partial\vec
A}{\partial t}-\vec {\tilde u}\times curl_{\vec x}\vec
A+\nabla_{\vec x}\left(\vec A\cdot\vec {\tilde
u}\right)\right)\right)\right\}=\\
\frac{4\pi}{c}\vec j+\frac{\partial}{\partial
t}\left\{\frac{1}{\gamma_0 c}\left(-\nabla_{\vec
x}\Psi_1-\frac{1}{c}\left(\frac{\partial\vec A}{\partial t}-\vec
{\tilde u}\times curl_{\vec x}\vec A+\nabla_{\vec x}\left(\vec
A\cdot\vec {\tilde u}\right)\right)\right)\right\}.
\end{multline}
In particular, by \er{MaxVacFullPPNmmmffGG111} and
\er{MaxVacFullPPNnnnffGG1111} we infer
\begin{multline}\label{MaxVacFullPPNnnnffGG111}
\frac{1}{\kappa_0}\,curl_{\vec x}
\left\{\kappa_0\,curl_{\vec x} \vec A\right\}=\\
\frac{4\pi}{\kappa_0 c}\left(\vec j-\rho\vec {\tilde u}\right)-
\frac{1}{\kappa_0}\left(\frac{\partial}{\partial
t}\left\{\frac{1}{\gamma_0 c}\nabla_{\vec
x}\Psi_1\right\}-curl_{\vec x}\left\{\vec {\tilde u}\times
\left\{\frac{1}{\gamma_0 c}\nabla_{\vec
x}\Psi_1\right\}\right\}+\left(\Div_{\vec x}\left\{\frac{1}{\gamma_0
c}\nabla_{\vec x}\Psi_1\right\}\right)\vec {\tilde u}\right)
\\-\frac{1}{\kappa_0}\frac{\partial}{\partial
t}\left\{\frac{1}{\gamma_0 c^2}\left(\frac{\partial\vec A}{\partial
t}-\vec {\tilde u}\times curl_{\vec x}\vec A+\nabla_{\vec
x}\left(\vec A\cdot\vec {\tilde u}\right)\right)\right\}\\+
\frac{1}{\kappa_0}\,curl_{\vec x} \left\{\vec {\tilde u}\times
\left(\frac{1}{\gamma_0 c^2}\left(\frac{\partial\vec A}{\partial
t}-\vec {\tilde u}\times curl_{\vec x}\vec A+\nabla_{\vec
x}\left(\vec A\cdot\vec {\tilde u}\right)\right)\right)\right\}\\
 -\frac{1}{\kappa_0}\left(\Div_{\vec
x}\left\{\frac{1}{\gamma_0 c^2}\left(\frac{\partial\vec A}{\partial
t}-\vec {\tilde u}\times curl_{\vec x}\vec A+\nabla_{\vec
x}\left(\vec A\cdot\vec {\tilde u}\right)\right)\right\}\right)\vec
{\tilde u}.
\end{multline}
Then, by \er{apfrm6}, \er{apfrm7} and \er{apfrm9} we rewrite
\er{MaxVacFullPPNmmmffGG111} and \er{MaxVacFullPPNnnnffGG111} as:
\begin{multline}\label{MaxVacFullPPNmmmffGG111jjkkhkjh}
\Div_{\vec x}\left\{\frac{1}{\gamma_0 c}\left(\frac{\partial\vec
A}{\partial t}-curl_{\vec x}\left\{\vec {\tilde u}\times \vec
A\right\}+\left(\Div_{\vec x}\vec A\right)\vec {\tilde
u}\right)\right\}+\Div_{\vec
x}\left\{\frac{1}{\gamma_0}\,\nabla_{\vec x}\Psi_1\right\}\\+
\Div_{\vec x}\left\{\frac{1}{\gamma_0 c}\left(d_{\vec x}\vec {\tilde
u}+\left\{d_{\vec x}\vec {\tilde u}\right\}^T\right)\cdot\vec
A-\frac{1}{\gamma_0 c}\left(\Div_{\vec x}\vec {\tilde u}\right)\vec
A\right\} =-4\pi\rho
\end{multline}
and
\begin{multline}\label{MaxVacFullPPNnnnffGG111ikiojioji}
\frac{1}{\kappa_0}\,curl_{\vec x} \left\{\kappa_0\,curl_{\vec x}
\vec A\right\} -\frac{4\pi}{\kappa_0 c}\left(\vec j-\rho\vec {\tilde
u}\right)\\+ \frac{1}{\kappa_0}\left(\frac{\partial}{\partial
t}\left\{\frac{1}{\gamma_0 c}\nabla_{\vec x}\Psi_1\right\}-\vec
{\tilde u}\times curl_{\vec x}\left\{\frac{1}{\gamma_0
c}\nabla_{\vec x}\Psi_1\right\}+\nabla_{\vec
x}\left(\left\{\frac{1}{\gamma_0 c}\nabla_{\vec
x}\Psi_1\right\}\cdot\vec {\tilde u}\right)\right) \\-
\frac{1}{\kappa_0}\left(\left(d_{\vec x}\vec {\tilde
u}+\left\{d_{\vec x}\vec {\tilde
u}\right\}^T\right)\cdot\left\{\frac{1}{\gamma_0 c}\nabla_{\vec
x}\Psi_1\right\}-\left(\Div_{\vec x}\vec {\tilde
u}\right)\left\{\frac{1}{\gamma_0 c}\nabla_{\vec
x}\Psi_1\right\}\right)
\\=-\frac{1}{\kappa_0}\frac{\partial}{\partial
t}\left\{\frac{1}{\gamma_0 c^2}\left(\frac{\partial\vec A}{\partial
t}-\vec {\tilde u}\times curl_{\vec x}\vec A+\nabla_{\vec
x}\left(\vec A\cdot\vec {\tilde u}\right)\right)\right\}\\-
\frac{1}{\kappa_0}\,\Div_{\vec x} \left\{ \left(\frac{1}{\gamma_0
c^2}\left(\frac{\partial\vec A}{\partial t}-\vec {\tilde u}\times
curl_{\vec x}\vec A+\nabla_{\vec
x}\left(\vec A\cdot\vec {\tilde u}\right)\right)\right)\otimes\vec {\tilde u}\right\}\\
 +\frac{1}{\kappa_0}\left(d_{\vec
x}\vec {\tilde u}\cdot\left\{\frac{1}{\gamma_0
c^2}\left(\frac{\partial\vec A}{\partial t}-\vec {\tilde u}\times
curl_{\vec x}\vec A+\nabla_{\vec x}\left(\vec A\cdot\vec {\tilde
u}\right)\right)\right\}\right)\\=-\frac{1}{\kappa_0}\frac{\partial}{\partial
t}\left\{\frac{1}{\gamma_0 c^2}\left(\frac{\partial\vec A}{\partial
t}+d_{\vec x}\vec A\cdot\vec {\tilde u}+\left\{d_{\vec x}\vec
{\tilde u}\right\}^T\cdot\vec A\right)\right\}\\-
\frac{1}{\kappa_0}\,\Div_{\vec x} \left\{ \left(\frac{1}{\gamma_0
c^2}\left(\frac{\partial\vec A}{\partial t}+d_{\vec x}\vec
A\cdot\vec {\tilde u}+\left\{d_{\vec x}\vec
{\tilde u}\right\}^T\cdot\vec A\right)\right)\otimes\vec {\tilde u}\right\}\\
 +\frac{1}{\kappa_0}\left(d_{\vec
x}\vec {\tilde u}\cdot\left\{\frac{1}{\gamma_0
c^2}\left(\frac{\partial\vec A}{\partial t}+d_{\vec x}\vec
A\cdot\vec {\tilde u}+\left\{d_{\vec x}\vec {\tilde
u}\right\}^T\cdot\vec A\right)\right\}\right).
\end{multline}

Next, up to the end of this section we study equation
\er{MaxVacFullPPNffGG} or equivalently
\er{MaxVacFullPPNmmmffGG111jjkkhkjh},
\er{MaxVacFullPPNnnnffGG111ikiojioji} in domains where we assume
that the coefficients $\gamma_0\neq 0$ and $\kappa_0\neq 0$ vary
sufficiently slow on the place and time and thus their spatial and
temporal derivatives are negligible with respect to other terms.
Then, we rewrite \er{MaxVacFullPPNmmmffGG111jjkkhkjh} and
\er{MaxVacFullPPNnnnffGG111ikiojioji} as
\begin{multline}\label{MaxVacFullPPNmmmffGG111jjkkhkjhdfjosjf}
\left(\frac{\partial}{\partial t}\left\{\frac{1}{c}\Div_{\vec x}\vec
A\right\}+\Div_{\vec x}\left\{\frac{1}{c}\left(\Div_{\vec x}\vec
A\right)\vec {\tilde u}\right\}\right)+\Delta_{\vec x}\Psi_1\\+
\Div_{\vec x}\left\{\frac{1}{c}\left(d_{\vec x}\vec {\tilde
u}+\left\{d_{\vec x}\vec {\tilde u}\right\}^T\right)\cdot\vec
A-\frac{1}{c}\left(\Div_{\vec x}\vec {\tilde u}\right)\vec A\right\}
=-4\pi\gamma_0\rho
\end{multline}
and
\begin{multline}\label{MaxVacFullPPNnnnffGG111ikiojiojiscjkhh}
-\Delta_{\vec x}\vec A =\frac{4\pi}{\kappa_0 c}\left(\vec j-\rho\vec
{\tilde u}\right)- \nabla_{\vec x}\left(\frac{1}{\kappa_0\gamma_0
c}\left(\frac{\partial\Psi_1}{\partial t}+\vec {\tilde
u}\cdot\nabla_{\vec x}\Psi_1\right)+\left(\Div_{\vec x}\vec
A\right)\right)
\\+ \left(\left(d_{\vec x}\vec {\tilde
u}+\left\{d_{\vec x}\vec {\tilde
u}\right\}^T\right)\cdot\left\{\frac{1}{\kappa_0\gamma_0
c}\nabla_{\vec x}\Psi_1\right\}-\left(\Div_{\vec x}\vec {\tilde
u}\right)\left\{\frac{1}{\kappa_0\gamma_0 c}\nabla_{\vec
x}\Psi_1\right\}\right)
\\-\frac{\partial}{\partial
t}\left\{\frac{1}{\gamma_0\kappa_0 c^2}\left(\frac{\partial\vec
A}{\partial t}+d_{\vec x}\vec A\cdot\vec {\tilde u}+\left\{d_{\vec
x}\vec {\tilde u}\right\}^T\cdot\vec A\right)\right\}\\- \Div_{\vec
x} \left\{\left(\frac{1}{\gamma_0\kappa_0
c^2}\left(\frac{\partial\vec A}{\partial t}+d_{\vec x}\vec
A\cdot\vec {\tilde u}+\left\{d_{\vec x}\vec
{\tilde u}\right\}^T\cdot\vec A\right)\right)\otimes\vec {\tilde u}\right\}\\
 +\left(d_{\vec
x}\vec {\tilde u}\cdot\left\{\frac{1}{\gamma_0\kappa_0
c^2}\left(\frac{\partial\vec A}{\partial t}+d_{\vec x}\vec
A\cdot\vec {\tilde u}+\left\{d_{\vec x}\vec {\tilde
u}\right\}^T\cdot\vec A\right)\right\}\right).
\end{multline}
Alternatively, we rewrite
\er{MaxVacFullPPNnnnffGG111ikiojiojiscjkhh} as:
\begin{multline}\label{MaxVacFullPPNnnnffGG111ikiojiojiscjkhhopop}
-\Delta_{\vec x}\vec A = \frac{4\pi}{\kappa_0 c}\left(\vec
j-\rho\vec {\tilde u}\right)-\nabla_{\vec x}\left(\Div_{\vec x}\vec
A\right)\\- \frac{1}{\gamma_0\kappa_0
c}\left(\frac{\partial}{\partial t}\left\{\nabla_{\vec
x}\Psi_1\right\}-curl_{\vec x}\left\{\vec {\tilde u}\times
\nabla_{\vec x}\Psi_1\right\}+\left(\Delta_{\vec x}\Psi_1\right)\vec
{\tilde u}\right)
\\-\frac{\partial}{\partial
t}\left\{\frac{1}{\gamma_0\kappa_0 c^2}\left(\frac{\partial\vec
A}{\partial t}+d_{\vec x}\vec A\cdot\vec {\tilde u}+\left\{d_{\vec
x}\vec {\tilde u}\right\}^T\cdot\vec A\right)\right\}\\- \Div_{\vec
x} \left\{\left(\frac{1}{\gamma_0\kappa_0
c^2}\left(\frac{\partial\vec A}{\partial t}+d_{\vec x}\vec
A\cdot\vec {\tilde u}+\left\{d_{\vec x}\vec
{\tilde u}\right\}^T\cdot\vec A\right)\right)\otimes\vec {\tilde u}\right\}\\
 +\left(d_{\vec
x}\vec {\tilde u}\cdot\left\{\frac{1}{\gamma_0\kappa_0
c^2}\left(\frac{\partial\vec A}{\partial t}+d_{\vec x}\vec
A\cdot\vec {\tilde u}+\left\{d_{\vec x}\vec {\tilde
u}\right\}^T\cdot\vec A\right)\right\}\right).
\end{multline}

Next if we assume the following calibration of the potentials:
\begin{equation}\label{MaxVacFullPPNjjjjffhhGGGGGG}
\Div_{\vec x}\vec A=0,
\end{equation}
then by \er{MaxVacFullPPNjjjjffhhGGGGGG},
\er{MaxVacFullPPNmmmffGG111jjkkhkjhdfjosjf},
\er{MaxVacFullPPNnnnffGG111ikiojiojiscjkhhopop} and \er{apfrm6} we
have
\begin{equation}\label{MaxVacFullPPNmmmffffffhhtygghGGGG}
-\Delta_{\vec x}\Psi_1=4\pi\gamma_0\rho+\Div_{\vec x}
\left\{\frac{1}{c}\left(d_{\vec x}\vec {\tilde u}+\left\{d_{\vec
x}\vec {\tilde u}\right\}^T\right)\cdot\vec
A-\frac{1}{c}\left(\Div_{\vec x}\vec {\tilde u}\right)\vec
A\right\},
\end{equation}
and
\begin{multline}\label{MaxVacFullPPNnnnffffffyuughjhjhjhhjjkjhkkjhhjhghGGGG}
-\Delta_{\vec x}\vec A = \frac{4\pi}{\kappa_0 c}\left(\vec j-
\frac{1}{4\pi\gamma_0}\left(\frac{\partial}{\partial
t}\left\{\nabla_{\vec x}\Psi_1\right\}-curl_{\vec x}\left\{\vec
{\tilde u}\times \nabla_{\vec
x}\Psi_1\right\}\right)\right)\\+\frac{1}{\gamma_0\kappa_0
c^2}\left( \Div_{\vec x} \left\{\left(d_{\vec x}\vec {\tilde
u}+\left\{d_{\vec x}\vec {\tilde u}\right\}^T\right)\cdot\vec
A-\left(\Div_{\vec x}\vec {\tilde u}\right)\vec A\right\}\right)\vec
{\tilde u}
\\-\frac{\partial}{\partial
t}\left\{\frac{1}{\gamma_0\kappa_0 c^2}\left(\frac{\partial\vec
A}{\partial t}+d_{\vec x}\vec A\cdot\vec {\tilde u}+\left\{d_{\vec
x}\vec {\tilde u}\right\}^T\cdot\vec A\right)\right\}\\- \Div_{\vec
x} \left\{\left(\frac{1}{\gamma_0\kappa_0
c^2}\left(\frac{\partial\vec A}{\partial t}+d_{\vec x}\vec
A\cdot\vec {\tilde u}+\left\{d_{\vec x}\vec
{\tilde u}\right\}^T\cdot\vec A\right)\right)\otimes\vec {\tilde u}\right\}\\
 +\left(d_{\vec
x}\vec {\tilde u}\cdot\left\{\frac{1}{\gamma_0\kappa_0
c^2}\left(\frac{\partial\vec A}{\partial t}+d_{\vec x}\vec
A\cdot\vec {\tilde u}+\left\{d_{\vec x}\vec {\tilde
u}\right\}^T\cdot\vec A\right)\right\}\right).
\end{multline}

On the other hand, if we assume the following alternative
calibration of the potentials:
\begin{equation}\label{MaxVacFullPPNjjjjffhhGG}
\frac{1}{\kappa_0\gamma_0 c}\left(\frac{\partial\Psi_1}{\partial
t}+\vec {\tilde u}\cdot\nabla_{\vec x}\Psi_1\right)+\Div_{\vec
x}\vec A=0,
\end{equation}
then by \er{MaxVacFullPPNmmmffGG111jjkkhkjhdfjosjf},
\er{MaxVacFullPPNnnnffGG111ikiojiojiscjkhh} we have
\begin{multline}\label{MaxVacFullPPNmmmffffffiuiuhjuGG}
\left(\frac{\partial}{\partial t}\left\{\frac{1}{\kappa_0\gamma_0
c^2}\left(\frac{\partial\Psi_1}{\partial t}+\vec {\tilde
u}\cdot\nabla_{\vec x}\Psi_1\right)\right\}+\Div_{\vec x}
\left\{\frac{1}{\kappa_0\gamma_0
c^2}\left(\frac{\partial\Psi_1}{\partial t}+\vec {\tilde
u}\cdot\nabla_{\vec x}\Psi_1\right)\vec {\tilde
u}\right\}\right)-\Delta_{\vec
x}\Psi_1\\=4\pi\gamma_0\rho+\Div_{\vec x}
\left\{\frac{1}{c}\left(d_{\vec x}\vec {\tilde u}+\left\{d_{\vec
x}\vec {\tilde u}\right\}^T\right)\cdot\vec
A-\frac{1}{c}\left(\Div_{\vec x}\vec {\tilde u}\right)\vec
A\right\},
\end{multline}
and
\begin{multline}\label{MaxVacFullPPNnnnffffffyuughjhjhjhhjjkjhkkjhujgGG}
-\Delta_{\vec x}\vec A= \frac{4\pi}{\kappa_0 c}\left(\vec j-\rho\vec
{\tilde u}\right)+ \left(\left(d_{\vec x}\vec {\tilde
u}+\left\{d_{\vec x}\vec {\tilde
u}\right\}^T\right)\cdot\left\{\frac{1}{\kappa_0\gamma_0
c}\nabla_{\vec x}\Psi_1\right\}-\left(\Div_{\vec x}\vec {\tilde
u}\right)\left\{\frac{1}{\kappa_0\gamma_0 c}\nabla_{\vec
x}\Psi_1\right\}\right)
\\-\frac{\partial}{\partial
t}\left\{\frac{1}{\gamma_0\kappa_0 c^2}\left(\frac{\partial\vec
A}{\partial t}+d_{\vec x}\vec A\cdot\vec {\tilde u}+\left\{d_{\vec
x}\vec {\tilde u}\right\}^T\cdot\vec A\right)\right\}\\- \Div_{\vec
x} \left\{\left(\frac{1}{\gamma_0\kappa_0
c^2}\left(\frac{\partial\vec A}{\partial t}+d_{\vec x}\vec
A\cdot\vec {\tilde u}+\left\{d_{\vec x}\vec
{\tilde u}\right\}^T\cdot\vec A\right)\right)\otimes\vec {\tilde u}\right\}\\
 +\left(d_{\vec
x}\vec {\tilde u}\cdot\left\{\frac{1}{\gamma_0\kappa_0
c^2}\left(\frac{\partial\vec A}{\partial t}+d_{\vec x}\vec
A\cdot\vec {\tilde u}+\left\{d_{\vec x}\vec {\tilde
u}\right\}^T\cdot\vec A\right)\right\}\right).
\end{multline}

Next, from now we assume that $\vec {\tilde u}$ varies sufficiently
slowly in space and time variables, so that the following
approximation is valid:
\begin{equation}\label{MaxVacFullPPNmmmffffffhhtygghGGGGyuhggghghghffgiooiiojhjjhhjuiuiugughhjggjjhkoouou}
\frac{\left| d_{\vec x}\vec {\tilde u}+\left\{d_{\vec x}\vec {\tilde
u}\right\}^T\right|}{ c}\ll \frac{|\nabla_{\vec x}\Psi_1|}{| \vec
A|+|\Psi_1|}+\frac{|\gamma_0\kappa_0||\Delta_{\vec x}\vec A|}{|
d_{\vec x}\vec A|+|\nabla_{\vec x}\Psi_1|}.
\end{equation}
In particular, if in some Cartesian coordinate system we have both
\begin{equation}\label{MaxVacFullPPNmmmffffffhhtygghGGGGyuhggghghghffgiooiiojhjjhhjuiuiugughhjggjjhkoououioio}
\frac{\left| d_{\vec x}\vec {\tilde u}+\left\{d_{\vec x}\vec {\tilde
u}\right\}^T\right|^2}{
\left|\vec {\tilde u}\right|^2}\ll \left(\frac{|\nabla_{\vec
x}\Psi_1|}{| \vec
A|+|\Psi_1|}\right)^2+\left(\frac{|\gamma_0\kappa_0||\Delta_{\vec
x}\vec A|}{| d_{\vec x}\vec A|+|\nabla_{\vec x}\Psi_1|}\right)^2,
\end{equation}
and
\begin{equation}\label{MaxVacFullPPNmmmffffffiuiuhjuughbghhiuijghghghhjhjhhghyuyhhjhjihikjjjjkjljklghhggghiioio}
\frac{\left|\vec {\tilde u}\right|^2}{
c^2}\ll 1,
\end{equation}
then
\er{MaxVacFullPPNmmmffffffhhtygghGGGGyuhggghghghffgiooiiojhjjhhjuiuiugughhjggjjhkoouou}
indeed holds! Furthermore, taking into the account
\er{MaxVacFullPPNmmmffffffhhtygghGGGGyuhggghghghffgiooiiojhjjhhjuiuiugughhjggjjhkoouou},
under the calibration \er{MaxVacFullPPNjjjjffhhGGGGGG}, we rewrite
\er{MaxVacFullPPNmmmffffffhhtygghGGGG} and
\er{MaxVacFullPPNnnnffffffyuughjhjhjhhjjkjhkkjhhjhghGGGG} as
\begin{equation}\label{MaxVacFullPPNmmmffffffhhtygghGGGGgg}
-\Delta_{\vec x}\Psi_1\approx 4\pi\gamma_0\rho,
\end{equation}
and
\begin{multline}\label{MaxVacFullPPNnnnffffffyuughjhjhjhhjjkjhkkjhhjhghGGGGgg}
-\Delta_{\vec x}\vec A\approx \frac{4\pi}{\kappa_0 c}\left(\vec j-
\frac{1}{4\pi\gamma_0}\left(\frac{\partial}{\partial
t}\left\{\nabla_{\vec x}\Psi_1\right\}-curl_{\vec x}\left\{\vec
{\tilde u}\times \nabla_{\vec x}\Psi_1\right\}\right)\right)
\\-\frac{\partial}{\partial
t}\left\{\frac{1}{\gamma_0\kappa_0 c^2}\left(\frac{\partial\vec
A}{\partial t}+d_{\vec x}\vec A\cdot\vec {\tilde u}+\left\{d_{\vec
x}\vec {\tilde u}\right\}^T\cdot\vec A\right)\right\}\\- \Div_{\vec
x} \left\{\left(\frac{1}{\gamma_0\kappa_0
c^2}\left(\frac{\partial\vec A}{\partial t}+d_{\vec x}\vec
A\cdot\vec {\tilde u}+\left\{d_{\vec x}\vec
{\tilde u}\right\}^T\cdot\vec A\right)\right)\otimes\vec {\tilde u}\right\}\\
 +\left(d_{\vec
x}\vec {\tilde u}\cdot\left\{\frac{1}{\gamma_0\kappa_0
c^2}\left(\frac{\partial\vec A}{\partial t}+d_{\vec x}\vec
A\cdot\vec {\tilde u}+\left\{d_{\vec x}\vec {\tilde
u}\right\}^T\cdot\vec A\right)\right\}\right).
\end{multline}
Note that, using Proposition \ref{yghgjtgyrtrt} we deduce that the
approximate equations \er{MaxVacFullPPNmmmffffffhhtygghGGGGgg} and
\er{MaxVacFullPPNnnnffffffyuughjhjhjhhjjkjhkkjhhjhghGGGGgg} are
still invariant under the change of inertial or non-inertial
cartesian coordinate system, provided that $\vec A$ is a proper
vector field and $\Psi_1$ is a proper scalar field.

On the other hand, taking into the account
\er{MaxVacFullPPNmmmffffffhhtygghGGGGyuhggghghghffgiooiiojhjjhhjuiuiugughhjggjjhkoouou},
under the calibration \er{MaxVacFullPPNjjjjffhhGG}, we rewrite
\er{MaxVacFullPPNmmmffffffiuiuhjuGG} and
\er{MaxVacFullPPNnnnffffffyuughjhjhjhhjjkjhkkjhujgGG} as
\begin{equation}\label{MaxVacFullPPNmmmffffffiuiuhjuGGgg}
\frac{1}{\kappa_0\gamma_0 c^2}\left(\frac{\partial}{\partial
t}\left(\frac{\partial\Psi_1}{\partial t}+\vec {\tilde
u}\cdot\nabla_{\vec x}\Psi_1\right)+div_{\vec x}
\left\{\left(\frac{\partial\Psi_1}{\partial t}+\vec {\tilde
u}\cdot\nabla_{\vec x}\Psi_1\right)\vec {\tilde
u}\right\}\right)-\Delta_{\vec x}\Psi_1\approx 4\pi\gamma_0\rho.
\end{equation}
and
\begin{multline}\label{MaxVacFullPPNnnnffffffyuughjhjhjhhjjkjhkkjhujgGGgg}
-\Delta_{\vec x}\vec A\approx \frac{4\pi}{\kappa_0 c}\left(\vec
j-\rho\vec {\tilde u}\right)-\frac{\partial}{\partial
t}\left\{\frac{1}{\gamma_0\kappa_0 c^2}\left(\frac{\partial\vec
A}{\partial t}+d_{\vec x}\vec A\cdot\vec {\tilde u}+\left\{d_{\vec
x}\vec {\tilde u}\right\}^T\cdot\vec A\right)\right\}\\- \Div_{\vec
x} \left\{\left(\frac{1}{\gamma_0\kappa_0
c^2}\left(\frac{\partial\vec A}{\partial t}+d_{\vec x}\vec
A\cdot\vec {\tilde u}+\left\{d_{\vec x}\vec
{\tilde u}\right\}^T\cdot\vec A\right)\right)\otimes\vec {\tilde u}\right\}\\
 +\left(d_{\vec
x}\vec {\tilde u}\cdot\left\{\frac{1}{\gamma_0\kappa_0
c^2}\left(\frac{\partial\vec A}{\partial t}+d_{\vec x}\vec
A\cdot\vec {\tilde u}+\left\{d_{\vec x}\vec {\tilde
u}\right\}^T\cdot\vec A\right)\right\}\right).
\end{multline}
Again note that, using Proposition \ref{yghgjtgyrtrt} we deduce that
the approximate equations \er{MaxVacFullPPNmmmffffffiuiuhjuGGgg} and
\er{MaxVacFullPPNnnnffffffyuughjhjhjhhjjkjhkkjhujgGGgg} are still
invariant under the change of inertial or non-inertial cartesian
coordinate system, provided that $\vec A$ is a proper vector field
and $\Psi_1$ is a proper scalar field.

In particular, assume that in some Cartesian coordinate system $(*)$
we have the following stronger than
\er{MaxVacFullPPNmmmffffffhhtygghGGGGyuhggghghghffgiooiiojhjjhhjuiuiugughhjggjjhkoouou}
approximation:
\begin{equation}\label{MaxVacFullPPNmmmffffffhhtygghGGGGyuhggghgh}
\frac{\left| d_{\vec x}\vec {\tilde u}\right|}{
c}\ll \frac{|\gamma_0\kappa_0|| \Delta_{\vec x}\vec A|
+\frac{1}{c}\left|d_{\vec x}\left\{\frac{\partial \vec A}{\partial
t}\right\}\right|+\frac{1}{c^2}\left|\frac{\partial^2 \vec
A}{\partial t^2}\right|+|\gamma_0\kappa_0|| \nabla^2_{\vec
x}\Psi_1|+\frac{1}{c}\left|\frac{\partial}{\partial
t}\left(\nabla_{\vec
x}\Psi_1\right)\right|+\frac{1}{c^2}\left|\frac{\partial^2
\Psi_1}{\partial t^2}\right|}{| d_{\vec x}\vec
A|+\frac{1}{c}\left|\frac{\partial \vec A}{\partial
t}\right|+|\nabla_{\vec x}\Psi_1|+\frac{1}{c}\left|\frac{\partial
\Psi_1}{\partial t}\right|},
\end{equation}
i.e. the field $\vec {\tilde u}$ changes in space much slower than
$(\Psi_1,\vec A)$. Estimation
\er{MaxVacFullPPNmmmffffffhhtygghGGGGyuhggghgh} holds especially
good for the electromagnetic waves of high frequency for example for
the visible light. However,
\er{MaxVacFullPPNmmmffffffhhtygghGGGGyuhggghgh} is still well for
almost every electromagnetic field we meet in the common life,
except probably the magnetic field of the Earth. Then, by
\er{MaxVacFullPPNmmmffffffiuiuhjuGGgg},
\er{MaxVacFullPPNnnnffffffyuughjhjhjhhjjkjhkkjhujgGGgg} and
\er{MaxVacFullPPNmmmffffffhhtygghGGGGyuhggghgh} we can write the
further approximating equations:
\begin{equation}\label{MaxVacFullPPNmmmffffffiuiuhjuGGggFG}
\left(\frac{\partial}{\partial
t}\left\{\frac{1}{c^2_0}\left(\frac{\partial\Psi_1}{\partial t}+\vec
{\tilde u}\cdot\nabla_{\vec x}\Psi_1\right)\right\}+\Div_{\vec x}
\left\{\frac{1}{c^2_0}\left(\frac{\partial\Psi_1}{\partial t}+\vec
{\tilde u}\cdot\nabla_{\vec x}\Psi_1\right)\vec {\tilde
u}\right\}\right)-\Delta_{\vec x}\Psi_1 \approx 4\pi\gamma_0\rho,
\end{equation}
and
\begin{equation}\label{MaxVacFullPPNnnnffffffyuughjhjhjhhjjkjhkkjhujgGGggFG}
\left(\frac{\partial}{\partial
t}\left\{\frac{1}{c^2_0}\left(\frac{\partial\vec A}{\partial
t}+d_{\vec x}\vec A\cdot\vec {\tilde u}\right)\right\}+\Div_{\vec x}
\left\{\frac{1}{c^2_0}\left(\frac{\partial\vec A}{\partial
t}+d_{\vec x}\vec A\cdot\vec {\tilde u}\right)\otimes\vec {\tilde
u}\right\}\right)-\Delta_{\vec x}\vec A\approx \frac{4\pi}{\kappa_0
c}\left(\vec j-\rho\vec {\tilde u}\right),
\end{equation}
where the scalar quantity $c_0$, defined by
\begin{equation}\label{gughhghfbvnbv}
c_0=c\sqrt{\kappa_0\gamma_0},
\end{equation}
is called speed of light in the medium. Note that, although the
approximate equations \er{MaxVacFullPPNmmmffffffiuiuhjuGGggFG} and
\er{MaxVacFullPPNnnnffffffyuughjhjhjhhjjkjhkkjhujgGGggFG} are
invariant under the Galilean Transformation, they are not invariant
under the more general change of non-inertial cartesian coordinate
system. However, \er{MaxVacFullPPNmmmffffffiuiuhjuGGggFG} and
\er{MaxVacFullPPNnnnffffffyuughjhjhjhhjjkjhkkjhujgGGggFG} are more
convenient than \er{MaxVacFullPPNmmmffffffiuiuhjuGGgg} and
\er{MaxVacFullPPNnnnffffffyuughjhjhjhhjjkjhkkjhujgGGgg}, since the
scalar potential $\Psi_1$ and every of the three scalar components
of the vector potential $\vec A$ in
\er{MaxVacFullPPNmmmffffffiuiuhjuGGggFG} and
\er{MaxVacFullPPNnnnffffffyuughjhjhjhhjjkjhkkjhujgGGggFG} satisfies
four decoupled equations of the same type, that differ only by the
right parts. On the other hand, if we consider some three proper
vector fields $\vec e_1:=\vec e_1(\vec x,t)$, $\vec e_2:=\vec
e_2(\vec x,t)$, and $\vec e_3:=\vec e_3(\vec x,t)$, which are
mutually orthogonal to each other and satisfy the following
approximation:
\begin{equation}\label{MaxVacFullPPNmmmffffffhhtygghGGGGyuhggghghhjhjh}
\frac{|d_{\vec x}\vec {e}_1|+c_0|\partial_t\vec e_1|}{|\vec
{e}_1|}+\frac{|d_{\vec x}\vec {e}_2|+c_0|\partial_t\vec e_2|}{|\vec
{e}_2|}+\frac{|d_{\vec x}\vec {e}_3|+c_0|\partial_t\vec e_3|}{|\vec
{e}_3|}\ll\frac{|\nabla_{\vec
x}\Psi_1|+c_0|\partial_t\Psi_1|+|d_{\vec x}\vec
A|+c_0|\partial_t\vec A|}{|\Psi_1|+|\vec A|}.
\end{equation}
in the given before coordinate system $(*)$, i.e. the field $\vec
{e}_k$ vary in space and time much weaker than $(\Psi_1,\vec A)$,
then we may write the alternative to
\er{MaxVacFullPPNnnnffffffyuughjhjhjhhjjkjhkkjhujgGGggFG} and
\er{MaxVacFullPPNmmmffffffiuiuhjuGGggFG} approximate equations in
the form of four decoupled scalar invariant wave equations of the
same type:
\begin{multline}\label{MaxVacFullPPNmmmffffffiuiuhjuGGggFGjjkkjj}
\left(\frac{\partial}{\partial
t}\left\{\frac{1}{c^2_0}\left(\frac{\partial}{\partial t}\left(\vec
e_k\cdot\vec A\right)+\vec {\tilde u}\cdot\nabla_{\vec x}\left(\vec
e_k\cdot\vec A\right)\right)\right\}+\Div_{\vec x}
\left\{\frac{1}{c^2_0}\left(\frac{\partial}{\partial t}\left(\vec
e_k\cdot\vec A\right)+\vec {\tilde u}\cdot\nabla_{\vec x}\left(\vec
e_k\cdot\vec A\right)\right)\vec {\tilde
u}\right\}\right)\\-\Delta_{\vec x}\left(\vec e_k\cdot\vec A\right)
\,\approx\, \frac{4\pi}{\kappa_0 c}\left(\vec j-\rho\vec {\tilde
u}\right)\cdot\vec e_k\quad\quad\forall k=1,2,3,
\end{multline}
and
\begin{equation}\label{MaxVacFullPPNmmmffffffiuiuhjuGGggFGkjkljl}
\left(\frac{\partial}{\partial
t}\left\{\frac{1}{c^2_0}\left(\frac{\partial\Psi_1}{\partial t}+\vec
{\tilde u}\cdot\nabla_{\vec x}\Psi_1\right)\right\}+\Div_{\vec x}
\left\{\frac{1}{c^2_0}\left(\frac{\partial\Psi_1}{\partial t}+\vec
{\tilde u}\cdot\nabla_{\vec x}\Psi_1\right)\vec {\tilde
u}\right\}\right)-\Delta_{\vec x}\Psi_1 \approx 4\pi\gamma_0\rho.
\end{equation}
Then, clearly, the new alternative approximate equations
\er{MaxVacFullPPNmmmffffffiuiuhjuGGggFGjjkkjj},
\er{MaxVacFullPPNmmmffffffiuiuhjuGGggFGkjkljl} are indeed invariant
under the more general change of non-inertial cartesian coordinate
system. So we can use approximate equations
\er{MaxVacFullPPNmmmffffffiuiuhjuGGggFGjjkkjj} and
\er{MaxVacFullPPNmmmffffffiuiuhjuGGggFGkjkljl} in the arbitrary
Cartesian coordinate system $(*)$ even if
\er{MaxVacFullPPNmmmffffffhhtygghGGGGyuhggghgh} and
\er{MaxVacFullPPNmmmffffffhhtygghGGGGyuhggghghhjhjh} are not
satisfied in the system $(*)$, provided that
\er{MaxVacFullPPNmmmffffffhhtygghGGGGyuhggghgh} and
\er{MaxVacFullPPNmmmffffffhhtygghGGGGyuhggghghhjhjh} are  satisfied
in another Cartesian system $(**)$.

In the absence of charges and currents (for example for
electromagnetic waves) equations
\er{MaxVacFullPPNmmmffffffiuiuhjuGGggFG} and
\er{MaxVacFullPPNnnnffffffyuughjhjhjhhjjkjhkkjhujgGGggFG} become:
\begin{equation}\label{MaxVacFullPPNmmmffffffiuiuhjuGGggFGel}
\left(\frac{\partial}{\partial
t}\left\{\frac{1}{c^2_0}\left(\frac{\partial\Psi_1}{\partial t}+\vec
{\tilde u}\cdot\nabla_{\vec x}\Psi_1\right)\right\}+\Div_{\vec x}
\left\{\frac{1}{c^2_0}\left(\frac{\partial\Psi_1}{\partial t}+\vec
{\tilde u}\cdot\nabla_{\vec x}\Psi_1\right)\vec {\tilde
u}\right\}\right)-\Delta_{\vec x}\Psi_1=0,
\end{equation}
and
\begin{equation}\label{MaxVacFullPPNnnnffffffyuughjhjhjhhjjkjhkkjhujgGGggFGel}
\left(\left\{\frac{1}{c^2_0}\frac{\partial}{\partial
t}\left(\frac{\partial\vec A}{\partial t}+d_{\vec x}\vec A\cdot\vec
{\tilde u}\right)\right\}+\Div_{\vec x}
\left\{\frac{1}{c^2_0}\left(\frac{\partial\vec A}{\partial
t}+d_{\vec x}\vec A\cdot\vec {\tilde u}\right)\otimes\vec {\tilde
u}\right\}\right)-\Delta_{\vec x}\vec A=0,
\end{equation}
and equations \er{MaxVacFullPPNmmmffffffiuiuhjuGGggFGjjkkjj},
\er{MaxVacFullPPNmmmffffffiuiuhjuGGggFGkjkljl} become:
\begin{multline}\label{MaxVacFullPPNmmmffffffiuiuhjuGGggFGjjkkjjkojjk}
\left(\frac{\partial}{\partial
t}\left\{\frac{1}{c^2_0}\left(\frac{\partial}{\partial t}\left(\vec
e_k\cdot\vec A\right)+\vec {\tilde u}\cdot\nabla_{\vec x}\left(\vec
e_k\cdot\vec A\right)\right)\right\}+\Div_{\vec x}
\left\{\frac{1}{c^2_0}\left(\frac{\partial}{\partial t}\left(\vec
e_k\cdot\vec A\right)+\vec {\tilde u}\cdot\nabla_{\vec x}\left(\vec
e_k\cdot\vec A\right)\right)\vec {\tilde
u}\right\}\right)\\-\Delta_{\vec x}\left(\vec e_k\cdot\vec A\right)
\,\approx\, 0\quad\quad\forall k=1,2,3,
\end{multline}
and
\begin{equation}\label{MaxVacFullPPNmmmffffffiuiuhjuGGggFGkjkljlh}
\left(\frac{\partial}{\partial
t}\left\{\frac{1}{c^2_0}\left(\frac{\partial\Psi_1}{\partial t}+\vec
{\tilde u}\cdot\nabla_{\vec x}\Psi_1\right)\right\}+\Div_{\vec x}
\left\{\frac{1}{c^2_0}\left(\frac{\partial\Psi_1}{\partial t}+\vec
{\tilde u}\cdot\nabla_{\vec x}\Psi_1\right)\vec {\tilde
u}\right\}\right)-\Delta_{\vec x}\Psi_1 \approx 0.
\end{equation}
Therefore, by \er{vhfffngghPPNffGG}, differentiating
\er{MaxVacFullPPNmmmffffffiuiuhjuGGggFGel} and
\er{MaxVacFullPPNnnnffffffyuughjhjhjhhjjkjhkkjhujgGGggFGel} or
\er{MaxVacFullPPNmmmffffffiuiuhjuGGggFGjjkkjjkojjk} and
\er{MaxVacFullPPNmmmffffffiuiuhjuGGggFGkjkljlh} and further usage of
\er{MaxVacFullPPNmmmffffffhhtygghGGGGyuhggghgh} and
\er{MaxVacFullPPNmmmffffffhhtygghGGGGyuhggghghhjhjh} gives that if
the scalar field $U:=U(\vec x,t)$ is one of any three scalar
components of every of the fields $\vec E$, $\vec B$, $\vec D$ or
$\vec H$, then $U$ satisfies the following approximate scalar
equation of the wave type:
\begin{equation}\label{MaxVacFullPPNmmmffffffiuiuhjuGGggFGelGHGHGHGG}
\left(\frac{\partial}{\partial
t}\left\{\frac{1}{c^2_0}\left(\frac{\partial U}{\partial t}+\vec
{\tilde u}\cdot\nabla_{\vec x}U\right)\right\}+\Div_{\vec x}
\left\{\frac{1}{c^2_0}\left(\frac{\partial U}{\partial t}+\vec
{\tilde u}\cdot\nabla_{\vec x} U\right)\vec {\tilde
u}\right\}\right)-\Delta_{\vec x}U \approx 0,
\end{equation}
where,
\begin{equation}\label{uyuyuyy}
\vec {\tilde u}=\left(\gamma_0\kappa_0\vec
v+\left(1-\gamma_0\kappa_0\right)\vec u\right)\,.
\end{equation}
Note that the wave equation
\er{MaxVacFullPPNmmmffffffiuiuhjuGGggFGelGHGHGHGG} is of the same
type as
\er{MaxVacFull1ninshtrgravortghhghgjkgghklhjgkghghjjkjhjkkggjkhjkhjjhhfhjhkjhjjhkhbbgjlmkmmhzzzyyykkknnnNWBWHWPPNjjlllkkkkkkljkkjjkjkhjjhhgghgggghjjhgjggh},
\er{MaxVacFull1ninshtrgravortghhghgjkgghklhjgkghghjjkjhjkkggjkhjkhjjhhfhjhkjhjjhkhbbgjlmkmmhzzzyyykkknnnNWBWHWPPNjjlllkkkkkkljkkjjkjkhjjhhgghgggghjjhgjgghjkjkjk},
\er{MaxVacFull1ninshtrgravortghhghgjkgghklhjgkghghjjkjhjkkggjkhjkhjjhhfhjhkjkhbbgjhzzzyyykkknnnNWBWHWPPNjjjljkjkkhhhhhhhhhhhhfnfffhuyuyuyutjhjhjhggjjfgfhfhjhhgjgjhkjhjhkhkhjj}
or \er{MaxVacFullPPNmmmffffffiuiuhjuGGggFGjjkkjjffjj}, with the only
difference that the field $\vec {\tilde u}$ appear in
\er{MaxVacFullPPNmmmffffffiuiuhjuGGggFGelGHGHGHGG} instead of the
field $\vec u_0$ in
\er{MaxVacFull1ninshtrgravortghhghgjkgghklhjgkghghjjkjhjkkggjkhjkhjjhhfhjhkjhjjhkhbbgjlmkmmhzzzyyykkknnnNWBWHWPPNjjlllkkkkkkljkkjjkjkhjjhhgghgggghjjhgjggh},
\er{MaxVacFull1ninshtrgravortghhghgjkgghklhjgkghghjjkjhjkkggjkhjkhjjhhfhjhkjhjjhkhbbgjlmkmmhzzzyyykkknnnNWBWHWPPNjjlllkkkkkkljkkjjkjkhjjhhgghgggghjjhgjgghjkjkjk}.
\er{MaxVacFull1ninshtrgravortghhghgjkgghklhjgkghghjjkjhjkkggjkhjkhjjhhfhjhkjkhbbgjhzzzyyykkknnnNWBWHWPPNjjjljkjkkhhhhhhhhhhhhfnfffhuyuyuyutjhjhjhggjjfgfhfhjhhgjgjhkjhjhkhkhjj}
or \er{MaxVacFullPPNmmmffffffiuiuhjuGGggFGjjkkjjffjj}.

\subsection{The case of quasistationary electromagnetic fields inside a slowly moving medium in a weak gravitational
field}\label{qfCM} Assume that in the given inertial or non-inertial
cartesian coordinate system $(*)$ the field $\vec {\tilde u}$ is
weak, meaning that at any instant on every point:
\begin{equation}\label{slowaetherGGaa}
\frac{1}{\kappa_0\gamma_0}\frac{|\vec {\tilde u}|^2}{c^2}\,\ll\, 1.
\end{equation}
Here $\vec {\tilde u}=\left(\gamma_0\kappa_0\vec
v+(1-\gamma_0\kappa_0)\vec u\right)$ is the speed-like vector field,
where $\vec v$ is a vectorial gravitational potential in the system
$(*)$ and $\vec u$ is the medium velocity. Furthermore, consider
quasistationary electromagnetic fields. This means the following:
assume that the changes in time of the physical characteristics of
the electromagnetic fields become essential after certain interval
of time $T_e$ and the changes in space of the physical
characteristics of the fields become essential in the spatial
landscape $L_e$. Then we assume that
\begin{equation}\label{slochangGGaa}
(\kappa_0\gamma_0)\frac{c^2T^2_e}{L^2_e}\,\gg\, 1.
\end{equation}
Next assume that we are under the calibration
\er{MaxVacFullPPNjjjjffhhGGGGGG}. Then by \er{slowaetherGGaa} and
\er{slochangGGaa} we rewrite \er{MaxVacFullPPNmmmffffffhhtygghGGGG}
and \er{MaxVacFullPPNnnnffffffyuughjhjhjhhjjkjhkkjhhjhghGGGG} as
\begin{equation}\label{MaxVacFullPPNmmmffffffhhtygghGGGGaa}
-\Delta_{\vec x}\Psi_1=4\pi\gamma_0\rho+\frac{1}{c}\,div_{\vec x}
\left\{\left(d_{\vec x}\vec {\tilde u}+\left\{d_{\vec x}\vec {\tilde
u}\right\}^T\right)\cdot\vec A-\left(div_{\vec x}\vec {\tilde
u}\right)\vec A\right\},
\end{equation}
and
\begin{equation}\label{MaxVacFullPPNnnnffffffyuughjhjhjhhjjkjhkkjhhjhghGGGGaa1}
-\Delta_{\vec x}\vec A\approx \frac{4\pi}{\kappa_0 c}\left(\vec
j-\rho\vec {\tilde u}\right)-\frac{1}{\kappa_0\gamma_0
c}\left(\frac{\partial}{\partial t}\left(\nabla_{\vec
x}\Psi_1\right)-curl_{\vec x}\left(\vec {\tilde u}\times\nabla_{\vec
x}\Psi_1\right)+\left(\Delta_{\vec x}\Psi_1\right)\vec {\tilde
u}\right).
\end{equation}
Moreover, by \er{slowaetherGGaa} and \er{slochangGGaa} we can
perform further approximation of
\er{MaxVacFullPPNnnnffffffyuughjhjhjhhjjkjhkkjhhjhghGGGGaa1} and we
get
\begin{multline}\label{MaxVacFullPPNnnnffffffyuughjhjhjhhjjkjhkkjhhjhghGGGGaa}
-\Delta_{\vec x}\vec A\approx\frac{4\pi}{\kappa_0 c}\left(\vec
j-\rho\vec {\tilde u}\right)-\frac{1}{\kappa_0\gamma_0
c}\left(\frac{\partial}{\partial t}\left(\nabla_{\vec
x}\Psi_1\right)-curl_{\vec x}\left(\vec {\tilde u}\times\nabla_{\vec
x}\Psi_1\right)+\left(\Delta_{\vec x}\Psi_1\right)\vec {\tilde
u}\right)
\\\approx\frac{4\pi}{\kappa_0 c}\,\vec
j-\frac{1}{\kappa_0 c}\left(\frac{\partial}{\partial
t}\left(\nabla_{\vec x}\psi_0\right)-curl_{\vec x}\left(\vec
v\times\nabla_{\vec x}\psi_0\right)\right),
\end{multline}
where $\psi_0(\vec x,t)$ is the classical Coulomb's potential which
satisfies
\begin{equation}\label{columbPPNaa}
-\Delta_{\vec x}\psi_0\equiv 4\pi\rho.
\end{equation}
So we rewrite \er{MaxVacFullPPNmmmffffffhhtygghGGGGaa} and
\er{MaxVacFullPPNnnnffffffyuughjhjhjhhjjkjhkkjhhjhghGGGGaa} as
\begin{equation}\label{MaxVacFull1bjkgjhjhgjgjgkjfhjfdghcgjhhjgkgkgugyyurhjfffhfjklhhhgkjgGGaaKK}
\begin{cases}
-\Delta_{\vec x}\vec A \approx\frac{4\pi}{\kappa_0 c}\vec {\widetilde j},\\
-\Delta_{\vec x}\Psi_1=4\pi\gamma_0\rho+\frac{1}{c}\,div_{\vec x}
\left\{\left(d_{\vec x}\vec {\tilde u}+\left\{d_{\vec x}\vec {\tilde
u}\right\}^T\right)\cdot\vec A-\left(div_{\vec x}\vec {\tilde
u}\right)\vec A\right\},
\end{cases}
\end{equation}
where we set the reduced current:
\begin{equation}\label{reducedcurrentfhfhjfhjGGaa}
\begin{cases}
\vec {\widetilde j}:=\vec j-\frac{1}{4\pi}\frac{\partial}{\partial
t} \left(\nabla_{\vec x}\psi_0\right)+\frac{1}{4\pi}curl_{\vec
x}\left(\vec {\tilde u}\times \nabla_{\vec x}\psi_0\right),\\
-\Delta_{\vec x}\psi_0= 4\pi\rho.
\end{cases}
\end{equation}
Note that by the Continuum Equation of the Conservation of Charges:
\begin{equation}\label{toksohraneniezarjadaPPNaa}
\frac{\partial\rho}{\partial t}+div_{\vec x}\vec j\equiv 0,
\end{equation}
the reduced current clearly satisfies:
\begin{equation}\label{divreducedcurrentPPNaa}
div_{\vec x}\vec {\widetilde j}\equiv 0.
\end{equation}
Moreover,  by \er{reducedcurrentfhfhjfhjGGaa} we clearly have
\begin{equation}\label{reducedcurrentPPNuighjhjaa}
\vec {\widetilde j}:=(\vec j-\rho\vec {\tilde
u})-\frac{1}{4\pi}\left(\frac{\partial}{\partial t}
\left(\nabla_{\vec x}\psi_0\right)-curl_{\vec x}\left(\vec {\tilde
u}\times \nabla_{\vec x}\psi_0\right)+\left(div_{\vec
x}\left\{\nabla_{\vec x}\psi_0\right\}\right)\vec {\tilde u}\right),
\end{equation}
and thus, by \er{reducedcurrentPPNuighjhjaa}, using Proposition
\ref{yghgjtgyrtrt}
we deduce that $\vec {\widetilde j}$ is a proper vector field.
Moreover, the approximate vectorial electromagnetic potential $\vec
A$ from
\er{MaxVacFull1bjkgjhjhgjgjgkjfhjfdghcgjhhjgkgkgugyyurhjfffhfjklhhhgkjgGGaaKK}
clearly satisfies:
\begin{equation}\label{MaxVacFull1bjkgjhjhgjgjgkjfhjfdghcgjhhjgkgkgugyyurkkkGGGGGaa}
div_{\vec x}\vec A=0.
\end{equation}
Next, since by \er{vhfffngghhjghhgPPNghghghutghffugghjhjkjjklgg} we
have:
\begin{equation}\label{vhfffngghhjghhgPPNghghghutghffugghjhjkjjklgghkhhh}
\Psi_1:=\Psi-\frac{1}{c}\vec A\cdot\vec {\tilde u},
\end{equation}
and, since by
\er{MaxVacFull1bjkgjhjhgjgjgkjfhjfdghcgjhhjgkgkgugyyurkkkGGGGGaa},
\er{apfrm6} and \er{apfrm9} we have:
\begin{multline}\label{MaxVacFullPPNnnnffffffyuughjhjhjhhjjkjhkkjhhjhghGGGGaajjkkjkjlj}
div_{\vec x} \left\{\left(d_{\vec x}\vec {\tilde u}+\left\{d_{\vec
x}\vec {\tilde u}\right\}^T\right)\cdot\vec A-\left(div_{\vec x}\vec
{\tilde u}\right)\vec A\right\}-\Delta_{\vec x}\left(\vec A\cdot\vec
{\tilde u}\right)=\\div_{\vec x} \left\{\left(d_{\vec x}\vec {\tilde
u}+\left\{d_{\vec x}\vec {\tilde u}\right\}^T\right)\cdot\vec
A-\left(div_{\vec x}\vec {\tilde u}\right)\vec A-\nabla_{\vec
x}\left(\vec A\cdot\vec {\tilde u}\right)\right\}=\\ div_{\vec x}
\left\{d_{\vec x}\vec {\tilde u}\cdot\vec A-d_{\vec x}\vec
A\cdot\vec {\tilde u}+\left(div_{\vec x}\vec A\right)\vec {\tilde
u}-\left(div_{\vec x}\vec {\tilde u}\right)\vec A-\vec {\tilde
u}\times curl_{\vec x}\vec A\right\}\\= div_{\vec x}
\left\{curl_{\vec x}\left(\vec {\tilde u}\times\vec A\right)-\vec
{\tilde u}\times curl_{\vec x}\vec A\right\}= -div_{\vec x}
\left\{\vec {\tilde u}\times curl_{\vec x}\vec A\right\},
\end{multline}
we rewrite
\er{MaxVacFull1bjkgjhjhgjgjgkjfhjfdghcgjhhjgkgkgugyyurhjfffhfjklhhhgkjgGGaaKK}
as:
\begin{equation}\label{MaxVacFull1bjkgjhjhgjgjgkjfhjfdghcgjhhjgkgkgugyyurhjfffhfjklhhhgkjgGGaa}
\begin{cases}
-\Delta_{\vec x}\vec A \approx\frac{4\pi}{\kappa_0 c}\vec {\widetilde j},\\
-\Delta_{\vec x}\Psi= 4\pi\gamma_0\rho-\frac{1}{c}\,div_{\vec
x}\left(\vec {\tilde u}\times curl_{\vec x} \vec A\right).
\end{cases}
\end{equation}
where
\begin{equation}\label{reducedcurrentfhfhjfhjGGaakklkl}
\begin{cases}
\vec {\widetilde j}:=\vec j-\frac{1}{4\pi}\frac{\partial}{\partial
t} \left(\nabla_{\vec x}\psi_0\right)+\frac{1}{4\pi}curl_{\vec
x}\left(\vec {\tilde u}\times \nabla_{\vec x}\psi_0\right),\\
-\Delta_{\vec x}\psi_0= 4\pi\rho.
\end{cases}
\end{equation}
So in order to find the scalar and the vectorial electromagnetic
potentials we just need to solve Laplace equations. Knowing the
approximate electromagnetic potentials by \er{vhfffngghPPN333yuyuGG}
we can find the approximations of of the electromagnetic fields:
\begin{equation}\label{MaxVacFull1bjkgjhjhgjgjgkjfhjfdghghligioiuittrhiguffGGaa}
\begin{cases}
\vec B= curl_{\vec x} \vec A\\
\vec E=-\nabla_{\vec x}\Psi-\frac{1}{c}\frac{\partial\vec
A}{\partial t}\\
 \vec D=-\frac{1}{\gamma_0}\nabla_{\vec
x}\Psi-\frac{1}{\gamma_0 c}\frac{\partial\vec A}{\partial t}+\frac{1}{c\gamma_0}\vec {\tilde u}\times curl_{\vec x}\vec A\\
\vec H=\kappa_0 \,curl_{\vec x} \vec A+\frac{1}{c}\,\vec {\tilde
u}\times\left(-\frac{1}{\gamma_0}\nabla_{\vec
x}\Psi-\frac{1}{\gamma_0 c}\frac{\partial\vec A}{\partial
t}+\frac{1}{\gamma_0 c}\vec {\tilde u}\times curl_{\vec x}\vec
A\right),
\end{cases}
\end{equation}
where $\Psi$ and $\vec A$ are given by
\er{MaxVacFull1bjkgjhjhgjgjgkjfhjfdghcgjhhjgkgkgugyyurhjfffhfjklhhhgkjgGGaa}.
Note also that, since $\vec {\widetilde j}$ is a proper vector
field, by Proposition \ref{yghgjtgyrtrt} we deduce that equations
\er{MaxVacFull1bjkgjhjhgjgjgkjfhjfdghcgjhhjgkgkgugyyurhjfffhfjklhhhgkjgGGaaKK}
and thus also equations
\er{MaxVacFull1bjkgjhjhgjgjgkjfhjfdghcgjhhjgkgkgugyyurhjfffhfjklhhhgkjgGGaa}
are invariant under the change of non-inertial cartesian coordinate
system, provided that $\vec A$ is a proper vector field and
$\Psi_1=\Psi-\frac{1}{c}\vec A\cdot\vec {\tilde u}$ is a proper
scalar field. So the approximate solutions in the case of
quasistationary fields in a weak gravitational field satisfy the
same transformation as the exact solutions of Maxwell Equations.
Therefore, if in coordinate system $(*)$ we can use the approximate
equations, given by
\er{MaxVacFull1bjkgjhjhgjgjgkjfhjfdghcgjhhjgkgkgugyyurhjfffhfjklhhhgkjgGGaa}
and \er{MaxVacFull1bjkgjhjhgjgjgkjfhjfdghghligioiuittrhiguffGGaa},
then we can use the similar approximation
also in coordinate system $(**)$, even in the case when in system
$(**)$ \er{slowaetherGGaa} or \er{slochangGGaa} are not satisfied.
\begin{remark}
The solutions of
\er{MaxVacFull1bjkgjhjhgjgjgkjfhjfdghcgjhhjgkgkgugyyurhjfffhfjklhhhgkjgGGaa}
and \er{MaxVacFull1bjkgjhjhgjgjgkjfhjfdghghligioiuittrhiguffGGaa}
satisfy the following equations:
\begin{equation}\label{MaxVacFull1bjkgjhjhgjaaajhfghhgGGaa}
\begin{cases}
curl_{\vec x} \left(\kappa_0\vec B+\frac{1}{c}\,\vec {\tilde
u}\times \left(- \nabla_{\vec x}\psi_0\right)\right)\equiv
\frac{4\pi}{c}\vec j+\frac{1}{c}\frac{\partial (-
\nabla_{\vec x}\psi_0)}{\partial t},\\
div_{\vec x} \vec D=4\pi\rho,\\
curl_{\vec x} \vec E+\frac{1}{c}\frac{\partial \vec B}{\partial t}=0,\\
div_{\vec x} \vec B=0\\
\vec E=\gamma_0\vec D-\frac{1}{c}\,\vec {\tilde u}\times \vec B\\
\vec H=\kappa_0\vec B+\frac{1}{c}\,\vec {\tilde u}\times \vec D,\\
\vec {\tilde u}=\left(\gamma_0\kappa_0\vec
v+(1-\gamma_0\kappa_0)\vec u\right),
\end{cases}
\end{equation}
where $\psi_0$ was defined by \er{columbPPNaa}. Equations
\er{MaxVacFull1bjkgjhjhgjaaajhfghhgGGaa} differ from the original
Maxwell equations \er{MaxVacFullPPNffGG} only by neglecting the
divergence-free part of the vector field $\vec D$ on the first
equation.
\end{remark}

Next, assume that, in addition to the validity of approximation
\er{slowaetherGGaa} and \er{slochangGGaa}, the approximation
\er{MaxVacFullPPNmmmffffffhhtygghGGGGyuhggghgh} also holds. Then we
further approximate
\er{MaxVacFull1bjkgjhjhgjgjgkjfhjfdghcgjhhjgkgkgugyyurhjfffhfjklhhhgkjgGGaaKK}
as:
\begin{equation}\label{MaxVacFull1bjkgjhjhgjgjgkjfhjfdghcgjhhjgkgkgugyyurhjfffhfjklhhhgkjgGGaaKKjkjj}
\begin{cases}
-\Delta_{\vec x}\Psi_1=4\pi\gamma_0\rho,\\
-\Delta_{\vec x}\vec A \approx \frac{4\pi}{\kappa_0 c}\,\vec
j-\frac{1}{\kappa_0\gamma_0 c}\left(\frac{\partial}{\partial
t}\left(\nabla_{\vec
x}\Psi_1\right)-curl_{\vec x}\left(\vec {\tilde u}\times\nabla_{\vec x}\Psi_1\right)\right)\\
\Psi=\Psi_1+\frac{1}{c}\vec A\cdot\vec {\tilde u}.
\end{cases}
\end{equation}
Moreover, as before, we deduce that equations
\er{MaxVacFull1bjkgjhjhgjgjgkjfhjfdghcgjhhjgkgkgugyyurhjfffhfjklhhhgkjgGGaaKKjkjj}
are also invariant under the change of non-inertial cartesian
coordinate system. Therefore, as before, if in coordinate system
$(*)$ we can use the approximation equations, given by
\er{MaxVacFull1bjkgjhjhgjgjgkjfhjfdghcgjhhjgkgkgugyyurhjfffhfjklhhhgkjgGGaaKKjkjj}
then we can use the similar equations
also in coordinate system $(**)$, even in the case when in system
$(**)$ \er{slowaetherGGaa}, \er{slochangGGaa} or
\er{MaxVacFullPPNmmmffffffhhtygghGGGGyuhggghgh} are not satisfied.

Finally, assume that we are under the alternative calibration
\er{MaxVacFullPPNjjjjffhhGG}. Then by \er{slowaetherGGaa} and
\er{slochangGGaa} we rewrite \er{MaxVacFullPPNmmmffffffiuiuhjuGG}
and \er{MaxVacFullPPNnnnffffffyuughjhjhjhhjjkjhkkjhujgGG} as:
\begin{equation}\label{MaxVacFullPPNmmmffffffiuiuhjuGGaa}
-\Delta_{\vec x}\Psi_1\approx
4\pi\gamma_0\rho+\frac{1}{c}\,div_{\vec x} \left\{\left(d_{\vec
x}\vec {\tilde u}+\left\{d_{\vec x}\vec {\tilde
u}\right\}^T\right)\cdot\vec A-\left(div_{\vec x}\vec {\tilde
u}\right)\vec A\right\},
\end{equation}
and
\begin{equation}\label{MaxVacFullPPNnnnffffffyuughjhjhjhhjjkjhkkjhujgGGaa}
-\Delta_{\vec x}\vec A\approx \frac{4\pi}{\kappa_0 c}\left(\vec
j-\rho\vec {\tilde u}\right)+\frac{1}{\kappa_0\gamma_0
c}\left(\left(d_{\vec x}\vec {\tilde u}+\left\{d_{\vec x}\vec
{\tilde u}\right\}^T\right)\cdot \nabla_{\vec
x}\Psi_1-\left(div_{\vec x}\vec {\tilde u}\right)\nabla_{\vec
x}\Psi_1\right).
\end{equation}
Thus if we assume that in addition to the approximation
\er{slowaetherGGaa} and \er{slochangGGaa} the approximation
\er{MaxVacFullPPNmmmffffffhhtygghGGGGyuhggghgh} also holds, we
further approximate \er{MaxVacFullPPNmmmffffffiuiuhjuGGaa} and
\er{MaxVacFullPPNnnnffffffyuughjhjhjhhjjkjhkkjhujgGGaa} as:
\begin{equation}\label{MaxVacFullPPNmmmffffffiuiuhjuGGaaghgghgh}
\begin{cases}
-\Delta_{\vec x}\Psi_1\approx 4\pi\gamma_0\rho,
\\
-\Delta_{\vec x}\vec A\approx \frac{4\pi}{\kappa_0 c}\left(\vec
j-\rho\vec {\tilde u}\right)\\
\Psi=\Psi_1+\frac{1}{c}\vec A\cdot\vec {\tilde u}.
\end{cases}
\end{equation}
Moreover, as before, we deduce that equations
\er{MaxVacFullPPNmmmffffffiuiuhjuGGaaghgghgh} are also invariant
under the change of non-inertial cartesian coordinate system.
Therefore, as before, if in coordinate system $(*)$ we can use the
approximation equations, given by
\er{MaxVacFullPPNmmmffffffiuiuhjuGGaaghgghgh}
then we can use the similar equations
also in coordinate system $(**)$, even in the case when in system
$(**)$ \er{slowaetherGGaa}, \er{slochangGGaa} or
\er{MaxVacFullPPNmmmffffffhhtygghGGGGyuhggghgh} are not satisfied.

\subsubsection{Hertz's notation for quasistationary electromagnetic fields}\label{gcCMuhhjhggh} Again consider in some domain the full system of
Maxwell equations in the vacuum or in a medium, without dispersion,
having the form \er{MaxVacFullPPNffGG}:
\begin{equation}\label{MaxVacFullPPNffGGvbccgc}
\begin{cases}
curl_{\vec x} \vec H=\frac{4\pi}{c}\vec j+
\frac{1}{c}\frac{\partial \vec D}{\partial t},\\
div_{\vec x} \vec D=4\pi\rho,\\
curl_{\vec x} \vec E+\frac{1}{c}\frac{\partial \vec B}{\partial t}=0,\\
div_{\vec x} \vec B=0\\
\vec E=\gamma_0\vec D-\frac{1}{c}\,\vec {\tilde u}\times \vec B\\
\vec H=\kappa_0\vec B+\frac{1}{c}\,\vec {\tilde u}\times \vec D,\\
\vec {\tilde u}=\left(\gamma_0\kappa_0\vec
v+(1-\gamma_0\kappa_0)\vec u\right).
\end{cases}
\end{equation}
Next define the electromagnetic fields in Hertz's notation:
\begin{equation}\label{MaxVacFullPPNffGGvbccgcjhhhjjihihjh}
\begin{cases}
\vec{\tilde{ E}}:=\vec E+\frac{1}{c}\,\vec {u}\times \vec B,\\
\vec{\tilde{ D}}:=\frac{1}{\gamma_0}\vec{\tilde{ E}}=\frac{1}{\gamma_0}\left(\vec E+\frac{1}{c}\,\vec {u}\times \vec B\right),\\
\vec{\tilde{ B}}:=\frac{1}{\kappa_0}\left(\vec H-\frac{1}{c}\,\vec {u}\times \vec D\right),\\
\vec{\tilde{ H}}:=\kappa_0\vec{\tilde{ B}}=\left(\vec
H-\frac{1}{c}\,\vec {u}\times \vec D\right).
\end{cases}
\end{equation}
Note here that $\vec u$ is the velocity field in the medium, thus
since this field is absent in vacuum we use the convention that for
the vacuum we have $\vec u:=\vec v$ (and in addition $\vec {\tilde
u}=\vec v$). In particular, by
\er{MaxVacFullPPNffGGvbccgcjhhhjjihihjh} and
\er{MaxVacFullPPNffGGvbccgc} we have
\begin{equation}\label{MaxVacFullPPNffGGvbccgcjhhhjjihihjhhhhiuu8hjiy}
\begin{cases}
\vec{\tilde{ E}}=\gamma_0\vec D+\frac{1}{c}\,\left(\vec {u}-\vec {\tilde u}\right)\times \vec B=\gamma_0\left(\vec D+\frac{\kappa_0}{c}\,\left(\vec {u}-\vec v\right)\times \vec B\right),\\
\vec{\tilde{ D}}=\vec D+\frac{\kappa_0}{c}\,\left(\vec {u}-\vec v\right)\times \vec B,\\
\vec{\tilde{ B}}=\frac{1}{\kappa_0}\left(\kappa_0\vec B-\frac{1}{c}\,\left(\vec {u}-\vec {\tilde u}\right)\times \vec D\right)=\vec B-\frac{\gamma_0}{c}\,\left(\vec {u}-\vec v\right)\times \vec D,\\
\vec{\tilde{ H}}=\kappa_0\vec
B-\frac{\gamma_0\kappa_0}{c}\,\left(\vec {u}-\vec v\right)\times
\vec D.
\end{cases}
\end{equation}
Therefore, since $\vec D,\vec B, (\vec u-\vec v)$ are all proper
vector fields, we deduce that $\vec{\tilde{ E}}, \vec{\tilde{ D}},
\vec{\tilde{ B}}, \vec{\tilde{ H}}$ are all \underline{proper}
vector fields. Furthermore, by
\er{MaxVacFullPPNffGGvbccgcjhhhjjihihjh} and
\er{MaxVacFullPPNffGGvbccgcjhhhjjihihjhhhhiuu8hjiy} we rewrite the
first four equations in \er{MaxVacFullPPNffGGvbccgc} as:
\begin{equation}\label{MaxVacFullPPNffGGvbccgcu8uuuuu}
\begin{cases}
curl_{\vec x} \vec{\tilde{ H}}=\frac{4\pi}{c}\vec j+
\frac{1}{c}\frac{\partial \vec D}{\partial t}-\frac{1}{c}\,curl_{\vec x}\left(\vec {u}\times \vec D\right),\\
div_{\vec x} \vec{\tilde{ D}}=4\pi\rho+\frac{1}{c}\,div_{\vec x} \left(\kappa_0\left(\vec {u}-\vec v\right)\times \vec B\right),\\
curl_{\vec x} \vec{\tilde{ E}}=-\left(\frac{1}{c}\frac{\partial \vec B}{\partial t}-\frac{1}{c}\,curl_{\vec x}\left(\vec {u}\times \vec B\right)\right),\\
div_{\vec x} \vec{\tilde{ B}}=-\frac{1}{c}\,div_{\vec x}
\left(\gamma_0\left(\vec {u}-\vec v\right)\times \vec D\right).
\end{cases}
\end{equation}
Then by inserting \er{apfrm3} into
\er{MaxVacFullPPNffGGvbccgcu8uuuuu} we deduce:
\begin{equation}\label{MaxVacFullPPNffGGvbccgcu8uuuuukjjkhh}
\begin{cases}
curl_{\vec x} \vec{\tilde{ H}}=\frac{4\pi}{c}\vec j+
\frac{1}{c}\frac{\partial \vec D}{\partial t}-\frac{1}{c}\,curl_{\vec x}\left(\vec {u}\times \vec D\right),\\
div_{\vec x} \vec{\tilde{ D}}=4\pi\rho-\frac{1}{c}\left(\vec {u}-\vec v\right)\cdot\,curl_{\vec x} \left(\kappa_0\vec B\right)+\frac{\kappa_0}{c}\,\vec B\cdot curl_{\vec x} \left(\left(\vec {u}-\vec v\right)\right),\\
curl_{\vec x} \vec{\tilde{ E}}=-\left(\frac{1}{c}\frac{\partial \vec B}{\partial t}-\frac{1}{c}\,curl_{\vec x}\left(\vec {u}\times \vec B\right)\right),\\
div_{\vec x} \vec{\tilde{ B}}=\frac{1}{c}\left(\vec {u}-\vec
v\right)\cdot\,curl_{\vec x} \left(\gamma_0\vec
D\right)-\frac{\gamma_0}{c}\,\vec D\cdot curl_{\vec x}
\left(\left(\vec {u}-\vec v\right)\right).
\end{cases}
\end{equation}
Therefore, inserting \er{MaxVacFullPPNffGGvbccgc} into the second
and the forth equations in \er{MaxVacFullPPNffGGvbccgcu8uuuuukjjkhh}
we deduce:
\begin{equation}\label{MaxVacFullPPNffGGvbccgcu8uuuuukjjkhhhihjhj}
\begin{cases}
curl_{\vec x} \vec{\tilde{ H}}=\frac{4\pi}{c}\vec j+
\frac{1}{c}\frac{\partial \vec D}{\partial t}-\frac{1}{c}\,curl_{\vec x}\left(\vec {u}\times \vec D\right),\\
div_{\vec x} \vec{\tilde{ D}}=4\pi\left(\rho-\frac{1}{c^2}\left(\vec
{u}-\vec v\right)\cdot\vec j\right)-\frac{1}{c^2}\left(\vec {u}-\vec
v\right)\cdot\left(
\frac{\partial \vec D}{\partial t}-curl_{\vec x}\left(\vec {\tilde u}\times \vec D\right)\right)+\frac{\kappa_0}{c}\,\vec B\cdot curl_{\vec x} \left(\left(\vec {u}-\vec v\right)\right),\\
curl_{\vec x} \vec{\tilde{ E}}=-\left(\frac{1}{c}\frac{\partial \vec B}{\partial t}-\frac{1}{c}\,curl_{\vec x}\left(\vec {u}\times \vec B\right)\right),\\
div_{\vec x} \vec{\tilde{ B}}=-\frac{1}{c^2}\left(\vec {u}-\vec
v\right)\cdot\left(\left(\frac{\partial \vec B}{\partial
t}-curl_{\vec x}\left(\vec {\tilde u}\times \vec
B\right)\right)\right)-\frac{\gamma_0}{c}\,\vec D\cdot curl_{\vec x}
\left(\left(\vec {u}-\vec v\right)\right).
\end{cases}
\end{equation}

Next, analogously with \er{slowaetherGGaa}, assume  that in the
given inertial or non-inertial cartesian coordinate system $(*)$ the
fields $\vec u$, $\vec v$, $\vec {\tilde u}$ are weak, meaning that
at any instant on every point:
\begin{equation}\label{slowaetherGGaagghgh}
\frac{1}{\kappa_0\gamma_0}\left(\frac{|\vec {
u}|^2}{c^2}+\frac{|\vec v|^2}{c^2}+\frac{|\vec {\tilde
u}|^2}{c^2}\right)\,\ll\, 1,
\end{equation}
and consider quasistationary electromagnetic fields, meaning the
following: assume that the changes in time of the physical
characteristics of the electromagnetic fields become essential after
certain interval of time $T_e$ and the changes in space of the
physical characteristics of the fields become essential in the
spatial landscape $L_e$. Then, analogously with \er{slochangGGaa},
we assume that
\begin{equation}\label{slochangGGaajhhjhj}
(\kappa_0\gamma_0)\frac{c^2T^2_e}{L^2_e}\,\gg\, 1.
\end{equation}
Moreover, assume that we have the following approximation analogous
to \er{MaxVacFullPPNmmmffffffhhtygghGGGGyuhggghgh}: if the changes
in space of the physical characteristics of the electromagnetic
fields become essential in the spatial landscape $L_e$ and the
changes in space of the fields $\vec u$, $\vec v$, $\vec {\tilde u}$
become essential in the spatial landscape $L_{u}$, then we assume
\begin{equation}\label{MaxVacFullPPNmmmffffffhhtygghGGGGyuhggghghffggfgf}
L_e\ll L_u,
\end{equation}
i.e. the fields $\vec u$, $\vec v$, $\vec {\tilde u}$ change in
space much slower than the electromagnetic fields. Then by
\er{slowaetherGGaagghgh}, \er{slochangGGaajhhjhj} and
\er{MaxVacFullPPNmmmffffffhhtygghGGGGyuhggghghffggfgf}, using the
second and the third equation in
\er{MaxVacFullPPNffGGvbccgcjhhhjjihihjhhhhiuu8hjiy} ,we approximate
\er{MaxVacFullPPNffGGvbccgcu8uuuuukjjkhhhihjhj} as:
\begin{equation}\label{MaxVacFullPPNffGGvbccgcu8uuuuukjjkhhhihjhjkhhkmgjhjhhkhjh}
\begin{cases}
curl_{\vec x} \vec{\tilde{ H}}\approx\frac{4\pi}{c}\vec j+
\frac{1}{c}\frac{\partial \vec{\tilde{ D}}}{\partial
t}-\frac{1}{c}\,curl_{\vec x}\left(\vec {u}\times \vec{\tilde{
D}}\right),\\
div_{\vec x} \vec{\tilde{ D}}\approx
4\pi\left(\rho-\frac{1}{c^2}\left(\vec {u}-\vec
v\right)\cdot\left(\vec
j-\rho\vec u\right)\right),\\
curl_{\vec x} \vec{\tilde{
E}}\approx-\left(\frac{1}{c}\frac{\partial \vec{\tilde{
B}}}{\partial t}-\frac{1}{c}\,curl_{\vec x}\left(\vec {u}\times
\vec{\tilde{
B}}\right)\right),\\
div_{\vec x} \vec{\tilde{ B}}\approx 0.
\end{cases}
\end{equation}
Moreover, denoting
\begin{equation}\label{MaxVacFullPPNffGGvbccgcjhhhjjihihjhkjkhkhhk}
\begin{cases}
\tilde\rho^*:=\rho-\frac{1}{c^2}\left(\vec {u}-\vec
v\right)\cdot\left(\vec j-\rho\vec u\right),\\
\vec {\tilde j}^*=\left(\vec j-\rho\vec u\right)+\tilde\rho^*\vec u,
\end{cases}
\end{equation}
again by \er{slowaetherGGaagghgh}, \er{slochangGGaajhhjhj} and
\er{MaxVacFullPPNmmmffffffhhtygghGGGGyuhggghghffggfgf}, using the
equality of the conservation of the charge
$\frac{\partial\rho}{\partial t}+div_{\vec x}\vec j=0$ we deduce:
\begin{equation}\label{MaxVacFullPPNffGGvbccgcjhhhjjihihjhkjkhkhhkhkhhjgg}
\begin{cases}
\vec j\approx\vec {\tilde j^*},
\\
\frac{\partial\tilde\rho^*}{\partial t}+div_{\vec x}\vec {\tilde
j^*}\approx 0.
\end{cases}
\end{equation}
Moreover, by \er{MaxVacFullPPNffGGvbccgcjhhhjjihihjhkjkhkhhk} we
easily obtain that $\tilde\rho^*$ is a proper scalar and $\left(\vec
{\tilde j^*}-\tilde\rho^*\vec u\right)$ is a proper vector. Finally,
by \er{slowaetherGGaagghgh},
\er{MaxVacFullPPNffGGvbccgcjhhhjjihihjhhhhiuu8hjiy} and
\er{MaxVacFullPPNffGGvbccgcjhhhjjihihjh} we derive the following
approximate equations:
\begin{equation}\label{MaxVacFullPPNffGGvbccgcjhhhjjihihjhjhjhjjg}
\begin{cases}
\vec B\approx\vec{\tilde{
B}}+\frac{1}{\kappa_0}\frac{1}{c}\,\left(\vec {u}-\vec
v\right)\times \vec{\tilde{ E}},\\
\vec E=\vec{\tilde{ E}}-\frac{1}{c}\,\vec {u}\times \vec{ B}.
\end{cases}
\end{equation}
So, by
\er{MaxVacFullPPNffGGvbccgcu8uuuuukjjkhhhihjhjkhhkmgjhjhhkhjh},
\er{MaxVacFullPPNffGGvbccgcjhhhjjihihjhkjkhkhhk},
\er{MaxVacFullPPNffGGvbccgcjhhhjjihihjhkjkhkhhkhkhhjgg} and
\er{MaxVacFullPPNffGGvbccgcjhhhjjihihjh} we deduce the following
approximate Maxwell equations in Hertz's form and write them
together with \er{MaxVacFullPPNffGGvbccgcjhhhjjihihjhjhjhjjg} as the
following two sets of equations:
\begin{equation}\label{MaxVacFullPPNffGGvbccgcu8uuuuukjjkhhhihjhjkhhkmgjhjhhkhjhhjjhjhjjh}
\begin{cases}
curl_{\vec x} \left(\kappa_0\vec{\tilde{
B}}\right)\approx\frac{4\pi}{c}\vec {\tilde j^*}+
\frac{1}{c}\frac{\partial }{\partial
t}\left(\frac{1}{\gamma_0}\vec{\tilde{
E}}\right)-\frac{1}{c}\,curl_{\vec x}\left(\vec
{u}\times\frac{1}{\gamma_0} \vec{\tilde{
E}}\right),\\
div_{\vec x} \left(\frac{1}{\gamma_0}\vec{\tilde{ E}}\right)\approx
4\pi\tilde\rho^*,\\
curl_{\vec x} \vec{\tilde{
E}}\approx-\left(\frac{1}{c}\frac{\partial \vec{\tilde{
B}}}{\partial t}-\frac{1}{c}\,curl_{\vec x}\left(\vec {u}\times
\vec{\tilde{
B}}\right)\right),\\
div_{\vec x} \vec{\tilde{ B}}\approx 0,\\
\frac{\partial\tilde\rho^*}{\partial t}+div_{\vec x}\vec {\tilde
j^*}\approx 0,
\end{cases}
\end{equation}
and
\begin{equation}\label{MaxVacFullPPNffGGvbccgcu8uuuuukjjkhhhihjhjkhhkmgjhjhhkhjhhjjhjhjjh11}
\begin{cases}
\rho=\tilde\rho^*+\frac{1}{c^2}\frac{1}{\kappa_0}\left(\vec {u}-\vec
v\right)\cdot\left(\vec {\tilde j^*}-\tilde\rho^*\vec u\right),\\
\vec j=\left(\vec {\tilde j^*}-\tilde\rho^*\vec u\right)+\rho\vec u,
\\
\vec B\approx\vec{\tilde{ B}}+\frac{1}{c}\,\left(\vec {u}-\vec
v\right)\times \vec{\tilde{ E}}\\
\vec E=\vec{\tilde{ E}}-\frac{1}{c}\,\vec {u}\times \vec{ B},
\end{cases}
\end{equation}
that are valid, however, only for quasistationary electromagnetic
fields. Next, since $\vec{\tilde{ E}}, \vec{\tilde{ D}},
\vec{\tilde{ B}}, \vec{\tilde{ H}}$, $\vec B,\left(\vec
E+\frac{1}{c}\,\vec {u}\times \vec{ B}\right)$ and $\left(\vec
{\tilde j^*}-\tilde\rho^*\vec u\right),\left(\vec {j}-\rho\vec
u\right),\left(\vec {u}-\vec v\right) $ are all proper vector fields
and $\tilde\rho^*,\rho$ are proper scalars, as before we deduce that
all equations in
\er{MaxVacFullPPNffGGvbccgcu8uuuuukjjkhhhihjhjkhhkmgjhjhhkhjhhjjhjhjjh}
and
\er{MaxVacFullPPNffGGvbccgcu8uuuuukjjkhhhihjhjkhhkmgjhjhhkhjhhjjhjhjjh11}
are invariant under the change of inertial or non-inertial cartesian
coordinate system. Therefore, if in coordinate system $(*)$ we can
use the approximate equations, given by
\er{MaxVacFullPPNffGGvbccgcu8uuuuukjjkhhhihjhjkhhkmgjhjhhkhjhhjjhjhjjh}
and
\er{MaxVacFullPPNffGGvbccgcu8uuuuukjjkhhhihjhjkhhkmgjhjhhkhjhhjjhjhjjh11},
then we can use the similar approximation
also in coordinate system $(**)$, even in the case when in system
$(**)$ \er{slowaetherGGaagghgh}, \er{slochangGGaajhhjhj} or
\er{MaxVacFullPPNmmmffffffhhtygghGGGGyuhggghghffggfgf} are not
satisfied.
Next, by \er{vjhfhjtjhjhuyyiyGG} we adjoint
\er{MaxVacFullPPNffGGvbccgcu8uuuuukjjkhhhihjhjkhhkmgjhjhhkhjhhjjhjhjjh}
and
\er{MaxVacFullPPNffGGvbccgcu8uuuuukjjkhhhihjhjkhhkmgjhjhhkhjhhjjhjhjjh11}
with the Ohm's Law that in our notations has the following simple
form:
\begin{equation}\label{vjhfhjtjhjhuyyiyGGhhgjg}
\vec {\tilde j^*}-\tilde\rho^*\,\vec u=\varepsilon\vec{\tilde E},
\end{equation}
and the Joules heat term, appearing in the First Law of
Thermodynamics
\er{MaxVacFull1ninshtrgravortghhghgjkgghklhjgkghghjjkjhjkkggjkhjkhjjhhfhjhkjkhbbgjhzzzyyykkknnnNWBWHWPPNjjhjhjhjh},
in our notations has the following simple form
\begin{equation}\label{vjhiiiihhjjhjh}\left(\vec j-\rho\vec u\right)\cdot\left(\vec E+\frac{1}{c}\vec
u\times\vec B\right)= \left(\vec {\tilde j^*}-\tilde\rho^*\,\vec
u\right)\cdot\vec{\tilde E}\,.
\end{equation}

Finally, note that the system
\er{MaxVacFullPPNffGGvbccgcu8uuuuukjjkhhhihjhjkhhkmgjhjhhkhjhhjjhjhjjh}
is completely analogous to the system \er{MaxVacFullPPNffGGvbccgc},
where the field $\frac{1}{\gamma_0}\vec{\tilde{ E}}$ substitutes the
field $\vec D$, the field $\vec{\tilde{ B}}$ substitutes the field
$\vec B$, $\tilde\rho^*,\vec{\tilde j^*}$ substitute $\rho,\vec{j}$,
and with the real velocity of the medium $\vec u$ instead of the
field $\vec {\tilde u}=\left(\gamma_0\kappa_0\vec
v+(1-\gamma_0\kappa_0)\vec u\right)$. On the other hand, in the case
of the real medium (not vacuum) the equations in
\er{MaxVacFullPPNffGGvbccgcu8uuuuukjjkhhhihjhjkhhkmgjhjhhkhjhhjjhjhjjh}
are completely independent of the vectorial gravitation potential
$\vec v$. Since the systems
\er{MaxVacFullPPNffGGvbccgcu8uuuuukjjkhhhihjhjkhhkmgjhjhhkhjhhjjhjhjjh}
and \er{MaxVacFullPPNffGGvbccgc} are similar they both obeys
wave-type solutions and thus one can wonder here: is it possible
that
\er{MaxVacFullPPNffGGvbccgcu8uuuuukjjkhhhihjhjkhhkmgjhjhhkhjhhjjhjhjjh}
is approximately equivalent to \er{MaxVacFullPPNffGGvbccgc} also for
highly oscillating in time electromagnetic fields and/or for
wave-type solutions in the absence of charges and currents, for
example in the case of the visible light? Unfortunately, the answer
for the last question should be negative, indeed as it can be seen,
the term $(1-\gamma_0\kappa_0)\vec u$ (consistently with the
equality $\vec {\tilde u}=\left(\gamma_0\kappa_0\vec
v+(1-\gamma_0\kappa_0)\vec u\right)$) appear in the the Fizeau
experiment for the light (see subsection \ref{seGOfz}) rather than
the term $\vec u$. The main reason for it is that
\er{slochangGGaajhhjhj} is not valid for highly oscillating in time
electromagnetic fields. Moreover, in the case of the complete
absence of charges and currents in the whole space
\er{slochangGGaajhhjhj} is not valid also for any nontrivial
wave-type solution, even with moderate frequency. So
\er{MaxVacFullPPNffGGvbccgcu8uuuuukjjkhhhihjhjkhhkmgjhjhhkhjhhjjhjhjjh}
and
\er{MaxVacFullPPNffGGvbccgcu8uuuuukjjkhhhihjhjkhhkmgjhjhhkhjhhjjhjhjjh11}
are valid, only for quasistationary electromagnetic fields,
generated by charges and currents moderately changing in time, for
example in the case of estimation of DC or AC chains.

Next, as before, again by the third and the fourth equations in
\er{MaxVacFullPPNffGGvbccgcu8uuuuukjjkhhhihjhjkhhkmgjhjhhkhjhhjjhjhjjh}
we can write
\begin{equation}\label{MaxVacFull1bjkgjhjhgjgjgkjfhjfdghghligioiuittrPPNggjoiuoui}
\begin{cases}
\vec {\tilde B}\approx curl_{\vec x} \vec {\tilde A^*},\\
\vec {\tilde E}\approx -\nabla_{\vec
x}\tilde\Psi^*-\frac{1}{c}\frac{\partial\vec {\tilde A^*}}{\partial
t}+\frac{1}{c}\,\vec {u}\times \left(curl_{\vec x} \vec {\tilde
A^*}\right),
\end{cases}
\end{equation}
where $\tilde\Psi^*$ and $\vec {\tilde A^*}$ are the scalar and the
vectorial electromagnetic potentials in Hertz's notation. Next,
analogously to \er{vhfffngghhjghhgPPNghghghutghffugghjhjkjjklgg} we
define the proper scalar electromagnetic potential in Hertz's
notation:
\begin{equation}\label{vhfffngghhjghhgPPNghghghutghffugghjhjkjjklgghihhih}
\tilde\Psi^*_1:=\tilde\Psi^*-\frac{1}{c}\vec {\tilde A^*}\cdot\vec
{u}.
\end{equation}
As before, the electromagnetic potentials in Hertz's notation are
not uniquely defined and thus we need to choose a calibration. For
definiteness we can take $\vec {\tilde A^*}$ to satisfy
\begin{equation}\label{MaxVacFull1bjkgjhjhgjgjgkjfhjfdghghligioiuittrPPN22jnhhjhjh}
div_{\vec x}\vec {\tilde A^*}= 0.
\end{equation}
Then, as before, $\vec {\tilde A^*}$ is a proper vector field and
$\tilde\Psi^*_1$ is a proper scalar field. It is clear that any
other choice of electromagnetic potentials with a different
calibration, as before, can be obtained by
\begin{equation}\label{MaxVacFull1bjkgjhjhgjgjgkjfhjfdghghligioiuittrPPNhjkjhkjgghhjjhjjhjgj}
\begin{cases}
\tilde\Psi^*\,\rightarrow\, \tilde\Psi^*+\frac{1}{c}\frac{\partial w}{\partial t}\\
\vec {\tilde A^*}\,\rightarrow\, \vec {\tilde A^*}-\nabla_{\vec x}w\\
\tilde\Psi^*_1\,\rightarrow\,
\tilde\Psi^*_1+\frac{1}{c}\left(\frac{\partial w}{\partial t}+\vec
u\cdot\nabla_{\vec x}w\right).
\end{cases}
\end{equation}
where $w:=w(\vec x,t)$ is an arbitrary proper scalar field.
Moreover, as before, under such a change of calibration, $\vec
{\tilde A^*}$ is still a proper vector field and $\tilde\Psi^*_1$ is
still a proper scalar field, provided the scalar field $w$ is
\underline{proper}. In particular the following alternative to
\er{MaxVacFull1bjkgjhjhgjgjgkjfhjfdghghligioiuittrPPN22jnhhjhjh}
calibration of the potentials can be chosen:
\begin{equation}\label{MaxVacFullPPNjjjjffhhGGljkjjkkjhk}
\frac{1}{\kappa_0\gamma_0
c}\left(\frac{\partial\tilde\Psi^*_1}{\partial t}+\vec
{u}\cdot\nabla_{\vec x}\tilde\Psi^*_1\right)+div_{\vec x}\vec
{\tilde A^*}=0.
\end{equation}

Next, inserting
\er{vhfffngghhjghhgPPNghghghutghffugghjhjkjjklgghihhih} into
\er{MaxVacFull1bjkgjhjhgjgjgkjfhjfdghghligioiuittrPPNggjoiuoui}
gives
\begin{equation}\label{MaxVacFull1bjkgjhjhgjgjgkjfhjfdghghligioiuittrPPNggjoiuouighujjhmngkh}
\begin{cases}
\vec {\tilde B}\approx curl_{\vec x} \vec {\tilde A^*},\\
\vec {\tilde E}\approx -\nabla_{\vec
x}\tilde\Psi^*_1-\frac{1}{c}\left(\frac{\partial\vec {\tilde
A^*}}{\partial t}-\vec {u}\times \left(curl_{\vec x} \vec {\tilde
A^*}\right)+\nabla_{\vec x}\left(\vec {\tilde A^*}\cdot\vec
{u}\right)\right).
\end{cases}
\end{equation}
Therefore, by \er{vfyutuyfffhhgfhgfh} and
\er{vhfffngghhjghhgjlkhjhkPPP} in Proposition \ref{yghgjtgyrtrt} we
rewrite
\er{MaxVacFull1bjkgjhjhgjgjgkjfhjfdghghligioiuittrPPNggjoiuouighujjhmngkh}
as:
\begin{equation}\label{MaxVacFull1bjkgjhjhgjgjgkjfhjfdghghligioiuittrPPNggjoiuouighujjhmngkhjkhjjhjh}
\begin{cases}
\vec {\tilde B}\approx curl_{\vec x} \vec {\tilde A^*},\\
\vec {\tilde E}\approx -\nabla_{\vec
x}\tilde\Psi^*_1-\frac{1}{c}\left(\frac{\partial\vec {\tilde
A^*}}{\partial t}-\,curl_{\vec x}\left(\vec {u}\times \vec {\tilde
A^*}\right)+\left(div_{\vec x}\vec {\tilde A^*}\right)\vec
{u}-\left(div_{\vec x}\vec u\right)\vec {\tilde A^*}+ \left(d_{\vec
x}\vec u+\left\{d_{\vec x}\vec u\right\}^T\right)\cdot\vec {\tilde
A^*}\right).
\end{cases}
\end{equation}
Thus, by \er{MaxVacFullPPNmmmffffffhhtygghGGGGyuhggghghffggfgf} we
approximate
\er{MaxVacFull1bjkgjhjhgjgjgkjfhjfdghghligioiuittrPPNggjoiuouighujjhmngkhjkhjjhjh}
as:
\begin{equation}\label{MaxVacFull1bjkgjhjhgjgjgkjfhjfdghghligioiuittrPPNggjoiuouighujjhmngkhjkhjjhjhkhkjh}
\begin{cases}
\vec {\tilde B}\approx curl_{\vec x} \vec {\tilde A^*},\\
\vec {\tilde E}\approx-\nabla_{\vec
x}\tilde\Psi^*_1-\frac{1}{c}\left(\frac{\partial\vec {\tilde
A^*}}{\partial t}-\,curl_{\vec x}\left(\vec {u}\times \vec {\tilde
A^*}\right)+\left(div_{\vec x}\vec {\tilde A^*}\right)\vec
{u}\right).
\end{cases}
\end{equation}
Furthermore, inserting
\er{MaxVacFull1bjkgjhjhgjgjgkjfhjfdghghligioiuittrPPNggjoiuouighujjhmngkh}
into
\er{MaxVacFullPPNffGGvbccgcu8uuuuukjjkhhhihjhjkhhkmgjhjhhkhjhhjjhjhjjh}
and using \er{slowaetherGGaagghgh}, \er{slochangGGaajhhjhj} and
\er{MaxVacFullPPNmmmffffffhhtygghGGGGyuhggghghffggfgf} we deduce
\begin{equation}\label{MaxVacFullPPNffGGvbccgcu8uuuuukjjkhhhihjhjkhhkmgjhjhhkhjhhjjhjhjjhjkkhhhhj}
\begin{cases}
curl_{\vec x} \left(\kappa_0\vec{\tilde{
B}}\right)\approx\frac{4\pi}{c}\left(\vec {\tilde
j^*}-\tilde\rho^*\vec u\right)- \frac{1}{c}\left(\frac{\partial
}{\partial t}\left(\frac{1}{\gamma_0}\nabla_{\vec
x}\tilde\Psi^*_1\right)-\,curl_{\vec x}\left(\vec
{u}\times\frac{1}{\gamma_0} \nabla_{\vec
x}\tilde\Psi^*_1\right)+\left(div_{\vec x}\left\{\frac{1}{\gamma_0}
\nabla_{\vec
x}\tilde\Psi^*_1\right\}\right)\vec u\right),\\
div_{\vec x} \left(\frac{1}{\gamma_0}\vec{\tilde{ E}}\right)\approx
4\pi\tilde\rho^*,\\
\frac{\partial\tilde\rho^*}{\partial t}+div_{\vec x}\vec {\tilde
j^*}\approx 0,\\
\vec {\tilde E}\approx-\nabla_{\vec
x}\tilde\Psi^*_1-\frac{1}{c}\left(\frac{\partial\vec {\tilde
A^*}}{\partial t}-\,curl_{\vec x}\left(\vec {u}\times \vec {\tilde
A^*}\right)+\left(div_{\vec x}\vec {\tilde A^*}\right)\vec
{u}\right),
\\ \vec {\tilde B}\approx curl_{\vec x} \vec {\tilde
A^*}.
\end{cases}
\end{equation}
Then again all equations in the new system
\er{MaxVacFullPPNffGGvbccgcu8uuuuukjjkhhhihjhjkhhkmgjhjhhkhjhhjjhjhjjhjkkhhhhj}
are invariant under the change of inertial or non-inertial cartesian
coordinate system.

Furthermore, if we assume that in the domain of the study the
coefficients $\gamma_0\neq 0$ and $\kappa_0\neq 0$ vary sufficiently
slow on the place and time and thus their spatial and temporal
derivatives are negligible with respect to other terms, then again
by \er{vfyutuyfffhhgfhgfh} and \er{vhfffngghhjghhgjlkhjhkPPP} in
Proposition \ref{yghgjtgyrtrt} and by
\er{MaxVacFullPPNmmmffffffhhtygghGGGGyuhggghghffggfgf},
\er{MaxVacFullPPNffGGvbccgcu8uuuuukjjkhhhihjhjkhhkmgjhjhhkhjhhjjhjhjjhjkkhhhhj}
can be approximately rewritten as:
\begin{equation}\label{MaxVacFullPPNffGGvbccgcu8uuuuukjjkhhhihjhjkhhkmgjhjhhkhjhhjjhjhjjh22}
\begin{cases}
\kappa_0\,curl_{\vec x} \vec{\tilde{
B}}\approx\frac{4\pi}{c}\left(\vec {\tilde j^*}-\tilde\rho^*\vec
u\right)- \nabla_{\vec
x}\left\{\frac{1}{c}\frac{1}{\gamma_0}\left(\frac{\partial\tilde\Psi^*_1}{\partial
t}+\vec
{u}\cdot\nabla_{\vec x}\tilde\Psi^*_1\right)\right\},\\
div_{\vec x} \vec{\tilde{ E}}\approx
4\pi\gamma_0\tilde\rho^*,\\
\frac{\partial\tilde\rho^*}{\partial t}+div_{\vec x}\vec {\tilde
j^*}\approx 0,\\
\vec {\tilde E}\approx-\nabla_{\vec
x}\tilde\Psi^*_1-\frac{1}{c}\left(\frac{\partial\vec {\tilde
A^*}}{\partial t}-\,curl_{\vec x}\left(\vec {u}\times \vec {\tilde
A^*}\right)+\left(div_{\vec x}\vec {\tilde A^*}\right)\vec
{u}\right),\\ \vec {\tilde B}\approx curl_{\vec x} \vec {\tilde
A^*}.
\end{cases}
\end{equation}
Furthermore, inserting the last equation of
\er{MaxVacFullPPNffGGvbccgcu8uuuuukjjkhhhihjhjkhhkmgjhjhhkhjhhjjhjhjjh22}
into the first one and using \er{apfrm12}, we rewrite
\er{MaxVacFullPPNffGGvbccgcu8uuuuukjjkhhhihjhjkhhkmgjhjhhkhjhhjjhjhjjh22}
as:
\begin{equation}\label{MaxVacFullPPNffGGvbccgcu8uuuuukjjkhhhihjhjkhhkmgjhjhhkhjhhjjhjhjjhjjjk}
\begin{cases}
-\Delta_{\vec x}\vec {\tilde
A^*}\approx\frac{4\pi}{\kappa_0c}\left(\vec {\tilde
j^*}-\tilde\rho^*\vec u\right)- \nabla_{\vec x}\left\{div_{\vec
x}\vec {\tilde
A^*}+\frac{1}{c}\frac{1}{\gamma_0\kappa_0}\left(\frac{\partial\tilde\Psi^*_1}{\partial
t}+\vec
{u}\cdot\nabla_{\vec x}\tilde\Psi^*_1\right)\right\},\\
div_{\vec x} \vec{\tilde{ E}}\approx
4\pi\gamma_0\tilde\rho^*,\\
\frac{\partial\tilde\rho^*}{\partial t}+div_{\vec x}\vec {\tilde
j^*}\approx 0,\\
\vec {\tilde E}\approx-\nabla_{\vec
x}\tilde\Psi^*_1-\frac{1}{c}\left(\frac{\partial\vec {\tilde
A^*}}{\partial t}-\,curl_{\vec x}\left(\vec {u}\times \vec {\tilde
A^*}\right)+\left(div_{\vec x}\vec {\tilde A^*}\right)\vec
{u}\right),\\
\vec {\tilde B}\approx curl_{\vec x} \vec {\tilde A^*}.
\end{cases}
\end{equation}
On the other hand, taking the divergence of both sides of the second
equation in
\er{MaxVacFull1bjkgjhjhgjgjgkjfhjfdghghligioiuittrPPNggjoiuouighujjhmngkhjkhjjhjhkhkjh}
gives
\begin{equation}\label{MaxVacFull1bjkgjhjhgjgjgkjfhjfdghghligioiuittrPPNggjoiuouighujjhmngkhjkhjjhjhkhkjhjhjhjg}
div_{\vec x}\vec {\tilde E}+\Delta_{\vec
x}\tilde\Psi_1\approx-\frac{1}{c}\left(\frac{\partial}{\partial
t}\left(div_{\vec x}\vec {\tilde A^*}\right)+div_{\vec
x}\left\{\left(div_{\vec x}\vec {\tilde A^*}\right)\vec
{u}\right\}\right).
\end{equation}
Inserting
\er{MaxVacFull1bjkgjhjhgjgjgkjfhjfdghghligioiuittrPPNggjoiuouighujjhmngkhjkhjjhjhkhkjhjhjhjg}
into
\er{MaxVacFullPPNffGGvbccgcu8uuuuukjjkhhhihjhjkhhkmgjhjhhkhjhhjjhjhjjhjjjk}
gives
\begin{equation}\label{MaxVacFullPPNffGGvbccgcu8uuuuukjjkhhhihjhjkhhkmgjhjhhkhjhhjjhjhjjhjjjkghhhf33}
\begin{cases}
-\Delta_{\vec x}\vec {\tilde
A^*}\approx\frac{4\pi}{\kappa_0c}\left(\vec {\tilde
j^*}-\tilde\rho^*\vec u\right)- \nabla_{\vec x}\left\{div_{\vec
x}\vec {\tilde
A^*}+\frac{1}{c}\frac{1}{\gamma_0\kappa_0}\left(\frac{\partial\tilde\Psi^*_1}{\partial
t}+\vec
{u}\cdot\nabla_{\vec x}\tilde\Psi^*_1\right)\right\},\\
-\Delta_{\vec x}\tilde\Psi_1\approx
4\pi\gamma_0\tilde\rho^*+\frac{1}{c}\left(\frac{\partial}{\partial
t}\left(div_{\vec x}\vec {\tilde A^*}\right)+div_{\vec
x}\left\{\left(div_{\vec x}\vec {\tilde A^*}\right)\vec
{u}\right\}\right),\\
\frac{\partial\tilde\rho^*}{\partial t}+div_{\vec x}\vec {\tilde
j^*}\approx 0,\\
\vec {\tilde E}\approx-\nabla_{\vec
x}\tilde\Psi^*_1-\frac{1}{c}\left(\frac{\partial\vec {\tilde
A^*}}{\partial t}-\,curl_{\vec x}\left(\vec {u}\times \vec {\tilde
A^*}\right)+\left(div_{\vec x}\vec {\tilde A^*}\right)\vec
{u}\right),\\
\vec {\tilde B}\approx curl_{\vec x} \vec {\tilde A^*}.
\end{cases}
\end{equation}
Therefore, by inserting
\er{MaxVacFull1bjkgjhjhgjgjgkjfhjfdghghligioiuittrPPNggjoiuouighujjhmngkh}
into
\er{MaxVacFullPPNffGGvbccgcu8uuuuukjjkhhhihjhjkhhkmgjhjhhkhjhhjjhjhjjhjjjkghhhf33}
we deduce:
\begin{equation}\label{MaxVacFullPPNffGGvbccgcu8uuuuukjjkhhhihjhjkhhkmgjhjhhkhjhhjjhjhjjhjjjkghhhf}
\begin{cases}
-\Delta_{\vec x}\vec {\tilde
A^*}\approx\frac{4\pi}{\kappa_0c}\left(\vec {\tilde
j^*}-\tilde\rho^*\vec u\right)- \nabla_{\vec x}\left\{div_{\vec
x}\vec {\tilde
A^*}+\frac{1}{c}\frac{1}{\gamma_0\kappa_0}\left(\frac{\partial\tilde\Psi^*_1}{\partial
t}+\vec
{u}\cdot\nabla_{\vec x}\tilde\Psi^*_1\right)\right\},\\
-\Delta_{\vec x}\tilde\Psi_1\approx
4\pi\gamma_0\tilde\rho^*+\frac{1}{c}\left(\frac{\partial}{\partial
t}\left(div_{\vec x}\vec {\tilde A^*}\right)+div_{\vec
x}\left\{\left(div_{\vec x}\vec {\tilde A^*}\right)\vec
{u}\right\}\right),\\
\frac{\partial\tilde\rho^*}{\partial t}+div_{\vec x}\vec {\tilde
j^*}\approx 0,\\
\vec {\tilde E}\approx -\nabla_{\vec
x}\tilde\Psi^*_1-\frac{1}{c}\left(\frac{\partial\vec {\tilde
A^*}}{\partial t}-\vec {u}\times \left(curl_{\vec x} \vec {\tilde
A^*}\right)+\nabla_{\vec x}\left(\vec {\tilde A^*}\cdot\vec
{u}\right)\right),\\
\vec {\tilde B}\approx curl_{\vec x} \vec {\tilde A^*}.
\end{cases}
\end{equation}
Thus, in the case of calibration
\er{MaxVacFull1bjkgjhjhgjgjgkjfhjfdghghligioiuittrPPN22jnhhjhjh} by
\er{vfyutuyfffhhgfhgfh} and \er{vhfffngghhjghhgjlkhjhkPPP} in
Proposition \ref{yghgjtgyrtrt}, using
\er{MaxVacFullPPNmmmffffffhhtygghGGGGyuhggghghffggfgf},
\er{MaxVacFullPPNffGGvbccgcu8uuuuukjjkhhhihjhjkhhkmgjhjhhkhjhhjjhjhjjhjjjkghhhf}
can be rewritten as:
\begin{equation}\label{MaxVacFullPPNffGGvbccgcu8uuuuukjjkhhhihjhjkhhkmgjhjhhkhjhhjjhjhjjhjjjkghhhfhjijhjhkhhkzz}
\begin{cases}
-\Delta_{\vec x}\tilde\Psi_1\approx
4\pi\gamma_0\tilde\rho^*,\\
-\Delta_{\vec x}\vec {\tilde
A^*}\approx\frac{4\pi}{\kappa_0c}\left(\vec {\tilde
j^*}-\tilde\rho^*\vec u\right)
-\frac{1}{c}\frac{1}{\gamma_0\kappa_0}\left(\frac{\partial
}{\partial t}\left(\nabla_{\vec x}\tilde\Psi^*_1\right)-\,curl_{\vec
x}\left(\vec {u}\times\nabla_{\vec
x}\tilde\Psi^*_1\right)+\left(\Delta_{\vec
x}\tilde\Psi^*_1\right)\vec u\right)
,\\
\frac{\partial\tilde\rho^*}{\partial t}+div_{\vec x}\vec {\tilde
j^*}\approx 0,\\
\vec {\tilde E}\approx -\nabla_{\vec
x}\tilde\Psi^*_1-\frac{1}{c}\left(\frac{\partial\vec {\tilde
A^*}}{\partial t}-\vec {u}\times \left(curl_{\vec x} \vec {\tilde
A^*}\right)+\nabla_{\vec x}\left(\vec {\tilde A^*}\cdot\vec
{u}\right)\right),
\\
\vec {\tilde B}\approx curl_{\vec x} \vec {\tilde A^*}.
\end{cases}
\end{equation}
On the other hand, under the alternative calibration
\er{MaxVacFullPPNjjjjffhhGGljkjjkkjhk}, again using
\er{slowaetherGGaagghgh}, \er{slochangGGaajhhjhj} and
\er{MaxVacFullPPNmmmffffffhhtygghGGGGyuhggghghffggfgf} in the second
equation of
\er{MaxVacFullPPNffGGvbccgcu8uuuuukjjkhhhihjhjkhhkmgjhjhhkhjhhjjhjhjjhjjjkghhhf},
we approximately rewrite
\er{MaxVacFullPPNffGGvbccgcu8uuuuukjjkhhhihjhjkhhkmgjhjhhkhjhhjjhjhjjhjjjkghhhf}
as:
\begin{equation}\label{MaxVacFullPPNffGGvbccgcu8uuuuukjjkhhhihjhjkhhkmgjhjhhkhjhhjjhjhjjhjjjkghhhfjkjkkhjggjzz}
\begin{cases}
-\Delta_{\vec x}\vec {\tilde
A^*}\approx\frac{4\pi}{\kappa_0c}\left(\vec {\tilde
j^*}-\tilde\rho^*\vec u\right),\\
-\Delta_{\vec x}\tilde\Psi_1\approx
4\pi\gamma_0\tilde\rho^*,\\
\frac{\partial\tilde\rho^*}{\partial t}+div_{\vec x}\vec {\tilde
j^*}\approx 0,\\
\vec {\tilde E}\approx -\nabla_{\vec
x}\tilde\Psi^*_1-\frac{1}{c}\left(\frac{\partial\vec {\tilde
A^*}}{\partial t}-\vec {u}\times \left(curl_{\vec x} \vec {\tilde
A^*}\right)+\nabla_{\vec x}\left(\vec {\tilde A^*}\cdot\vec
{u}\right)\right),
\\
\vec {\tilde B}\approx curl_{\vec x} \vec {\tilde A^*}.
\end{cases}
\end{equation}
Furthermore, defining the conduction currents:
\begin{equation}\label{uiyuyyyyyyyyyyyyyyyyyyyyyyyyyyyyyyyyyyyyyyyyyyyyyyyklj}
\begin{cases}
\vec {j}_{cond}:=\vec {j}-\rho\vec u,\\
\vec {\tilde j^*}_{cond}:=\vec {\tilde j^*}-\tilde\rho^*\vec u,
\end{cases}
\end{equation}
we obviously have:
\begin{equation}\label{uiyuyyyyyyyyyyyyyyyyyyyyyyyyyyyyyyyyyyyyyyyyyyyyyyykljggngn}
\begin{cases}
\vec {j}_{cond}= \vec {\tilde j^*}_{cond},\\
\frac{\partial\rho}{\partial t}=div_{\vec x}\left\{\rho\vec
u\right\}=-div_{\vec x}\left\{\vec {j}_{cond}\right\},
\\
\frac{\partial\tilde\rho^*}{\partial t}+div_{\vec
x}\left\{\tilde\rho^*\vec u\right\}\approx-div_{\vec x}\left\{\vec
{\tilde j^*}_{cond}\right\},
\end{cases}
\end{equation}
and moreover, obviously $\vec {j}_{cond}=\vec {\tilde j^*}_{cond}$
is a \underline{proper} vector field. Thus, by
\er{uiyuyyyyyyyyyyyyyyyyyyyyyyyyyyyyyyyyyyyyyyyyyyyyyyyklj} and
\er{uiyuyyyyyyyyyyyyyyyyyyyyyyyyyyyyyyyyyyyyyyyyyyyyyyykljggngn} we
rewrite
\er{MaxVacFullPPNffGGvbccgcu8uuuuukjjkhhhihjhjkhhkmgjhjhhkhjhhjjhjhjjhjjjkghhhfhjijhjhkhhkzz}
as:
\begin{equation}\label{MaxVacFullPPNffGGvbccgcu8uuuuukjjkhhhihjhjkhhkmgjhjhhkhjhhjjhjhjjhjjjkghhhfhjijhjhkhhk}
\begin{cases}
-\Delta_{\vec x}\tilde\Psi_1\approx
4\pi\gamma_0\tilde\rho^*,\\
-\Delta_{\vec x}\vec {\tilde A^*}\approx\frac{4\pi}{\kappa_0c}\vec
{\tilde j^*}_{cond}
-\frac{1}{c}\frac{1}{\gamma_0\kappa_0}\left(\frac{\partial
}{\partial t}\left(\nabla_{\vec x}\tilde\Psi^*_1\right)-\,curl_{\vec
x}\left(\vec {u}\times\nabla_{\vec
x}\tilde\Psi^*_1\right)+\left(\Delta_{\vec
x}\tilde\Psi^*_1\right)\vec u\right)
,\\
\frac{\partial\tilde\rho^*}{\partial t}+div_{\vec
x}\left\{\tilde\rho^*\vec u\right\}\approx-div_{\vec x}\left\{\vec
{\tilde j^*}_{cond}\right\},\\
\vec {\tilde E}\approx -\nabla_{\vec
x}\tilde\Psi^*_1-\frac{1}{c}\left(\frac{\partial\vec {\tilde
A^*}}{\partial t}-\vec {u}\times \left(curl_{\vec x} \vec {\tilde
A^*}\right)+\nabla_{\vec x}\left(\vec {\tilde A^*}\cdot\vec
{u}\right)\right),
\\
\vec {\tilde B}\approx curl_{\vec x} \vec {\tilde A^*}.
\end{cases}
\end{equation}
On the other hand, by
\er{uiyuyyyyyyyyyyyyyyyyyyyyyyyyyyyyyyyyyyyyyyyyyyyyyyyklj} and
\er{uiyuyyyyyyyyyyyyyyyyyyyyyyyyyyyyyyyyyyyyyyyyyyyyyyykljggngn} we
rewrite
\er{MaxVacFullPPNffGGvbccgcu8uuuuukjjkhhhihjhjkhhkmgjhjhhkhjhhjjhjhjjhjjjkghhhfjkjkkhjggjzz},
which was obtained under the alternative calibration, as:
\begin{equation}\label{MaxVacFullPPNffGGvbccgcu8uuuuukjjkhhhihjhjkhhkmgjhjhhkhjhhjjhjhjjhjjjkghhhfjkjkkhjggj}
\begin{cases}
-\Delta_{\vec x}\vec {\tilde A^*}\approx\frac{4\pi}{\kappa_0c}\vec
{\tilde j^*}_{cond},\\
-\Delta_{\vec x}\tilde\Psi_1\approx
4\pi\gamma_0\tilde\rho^*,\\
\frac{\partial\tilde\rho^*}{\partial t}+div_{\vec
x}\left\{\tilde\rho^*\vec u\right\}\approx-div_{\vec x}\left\{\vec
{\tilde j^*}_{cond}\right\},\\
\vec {\tilde E}\approx -\nabla_{\vec
x}\tilde\Psi^*_1-\frac{1}{c}\left(\frac{\partial\vec {\tilde
A^*}}{\partial t}-\vec {u}\times \left(curl_{\vec x} \vec {\tilde
A^*}\right)+\nabla_{\vec x}\left(\vec {\tilde A^*}\cdot\vec
{u}\right)\right),
\\
\vec {\tilde B}\approx curl_{\vec x} \vec {\tilde A^*}.
\end{cases}
\end{equation}
In both given cases of calibration in
\er{MaxVacFullPPNffGGvbccgcu8uuuuukjjkhhhihjhjkhhkmgjhjhhkhjhhjjhjhjjhjjjkghhhfhjijhjhkhhk}
or
\er{MaxVacFullPPNffGGvbccgcu8uuuuukjjkhhhihjhjkhhkmgjhjhhkhjhhjjhjhjjhjjjkghhhfjkjkkhjggj}
we just need to solve two Laplace equations in order to find
$\tilde\Psi_1$ and $\vec {\tilde A^*}$ and then we find $\vec
{\tilde E}$ and $\vec {\tilde B}$ by differentiation. Moreover, in
 either the case of
\er{MaxVacFullPPNffGGvbccgcu8uuuuukjjkhhhihjhjkhhkmgjhjhhkhjhhjjhjhjjhjjjkghhhfhjijhjhkhhk}
or the case of
\er{MaxVacFullPPNffGGvbccgcu8uuuuukjjkhhhihjhjkhhkmgjhjhhkhjhhjjhjhjjhjjjkghhhfjkjkkhjggj}
we have
\begin{equation}\label{MaxVagjkgggggggggggggjhjg}
div_{\vec x}\left\{\vec {\tilde E}+\nabla_{\vec
x}\tilde\Psi^*_1\right\}\approx 0\,
\end{equation}
and the Ohm's Law in the form of \er{vjhfhjtjhjhuyyiyGGhhgjg}:
\begin{equation}\label{vjhfhjtjhjhuyyiyGGhhgjgkhhhhhj}
\vec {\tilde j^*}_{cond}=\varepsilon\vec{\tilde E}.
\end{equation}
Furthermore, the Joules heat term, has the following simple form
\begin{equation}\label{vjhiiiihhjjhjhkkk}
\vec {\tilde j^*}_{cond}\cdot\vec{\tilde E}\,.
\end{equation}
Moreover, in both cases of
\er{MaxVacFullPPNffGGvbccgcu8uuuuukjjkhhhihjhjkhhkmgjhjhhkhjhhjjhjhjjhjjjkghhhfhjijhjhkhhk}
or
\er{MaxVacFullPPNffGGvbccgcu8uuuuukjjkhhhihjhjkhhkmgjhjhhkhjhhjjhjhjjhjjjkghhhfjkjkkhjggj},
by the last two equations there we obviously have:
\begin{equation}\label{gjkgjkgjkgjkgjkgjkgjkgjkgjkgjkgjkgjkgjkgjkgjkgjkgjkgjkgjkgjkgjkhjjh}
\begin{cases}
curl_{\vec x} \vec{\tilde{
E}}\approx-\left(\frac{1}{c}\frac{\partial \vec{\tilde{
B}}}{\partial t}-\frac{1}{c}\,curl_{\vec x}\left(\vec {u}\times
\vec{\tilde{
B}}\right)\right),\\
div_{\vec x} \vec{\tilde{ B}}\approx 0.
\end{cases}
\end{equation}

Next, in order to find the real electromagnetic fields in the usual
notation we have
\er{MaxVacFullPPNffGGvbccgcu8uuuuukjjkhhhihjhjkhhkmgjhjhhkhjhhjjhjhjjh11},
i.e. the following:
\begin{equation}\label{MaxVacFullPPNffGGvbccgcu8uuuuukjjkhhhihjhjkhhkmgjhjhhkhjhhjjhjhjjh11ljkjkljklj}
\begin{cases}
\vec B\approx\vec{\tilde{ B}}+\frac{1}{c}\,\left(\vec {u}-\vec
v\right)\times \vec{\tilde{ E}},\\
\vec E=\vec{\tilde{ E}}-\frac{1}{c}\,\vec {u}\times \vec{ B},\\
\rho=\tilde\rho^*+\frac{1}{c^2}\left(\vec {u}-\vec v\right)\cdot\vec
{\tilde j^*}_{cond},\\
\vec {j}_{cond}= \vec {\tilde j^*}_{cond},\\
\vec j=\vec {\tilde j^*}_{cond}+\rho\vec u.
\end{cases}
\end{equation}
Finally, as before, we deduce that all equations in either
\er{MaxVacFullPPNffGGvbccgcu8uuuuukjjkhhhihjhjkhhkmgjhjhhkhjhhjjhjhjjhjjjkghhhfhjijhjhkhhk}
or
\er{MaxVacFullPPNffGGvbccgcu8uuuuukjjkhhhihjhjkhhkmgjhjhhkhjhhjjhjhjjhjjjkghhhfjkjkkhjggj},
and, in addition, \er{MaxVagjkgggggggggggggjhjg},
\er{vjhfhjtjhjhuyyiyGGhhgjgkhhhhhj},
\er{gjkgjkgjkgjkgjkgjkgjkgjkgjkgjkgjkgjkgjkgjkgjkgjkgjkgjkgjkgjkgjkhjjh}
and
\er{MaxVacFullPPNffGGvbccgcu8uuuuukjjkhhhihjhjkhhkmgjhjhhkhjhhjjhjhjjh11ljkjkljklj}
are invariant under the change of inertial or non-inertial cartesian
coordinate system. Therefore, if in coordinate system $(*)$ we can
use the approximate equations, given by either
\er{MaxVacFullPPNffGGvbccgcu8uuuuukjjkhhhihjhjkhhkmgjhjhhkhjhhjjhjhjjhjjjkghhhfhjijhjhkhhk}
or
\er{MaxVacFullPPNffGGvbccgcu8uuuuukjjkhhhihjhjkhhkmgjhjhhkhjhhjjhjhjjhjjjkghhhfjkjkkhjggj},
and \er{MaxVagjkgggggggggggggjhjg},
\er{vjhfhjtjhjhuyyiyGGhhgjgkhhhhhj},
\er{gjkgjkgjkgjkgjkgjkgjkgjkgjkgjkgjkgjkgjkgjkgjkgjkgjkgjkgjkgjkgjkhjjh}
and
\er{MaxVacFullPPNffGGvbccgcu8uuuuukjjkhhhihjhjkhhkmgjhjhhkhjhhjjhjhjjh11ljkjkljklj},
then we can use the similar approximation
also in coordinate system $(**)$, even in the case when in system
$(**)$ \er{slowaetherGGaagghgh}, \er{slochangGGaajhhjhj} or
\er{MaxVacFullPPNmmmffffffhhtygghGGGGyuhggghghffggfgf} are not
satisfied.

Next, note that either
\er{MaxVacFullPPNffGGvbccgcu8uuuuukjjkhhhihjhjkhhkmgjhjhhkhjhhjjhjhjjhjjjkghhhfhjijhjhkhhk}
or
\er{MaxVacFullPPNffGGvbccgcu8uuuuukjjkhhhihjhjkhhkmgjhjhhkhjhhjjhjhjjhjjjkghhhfjkjkkhjggj},
and \er{MaxVagjkgggggggggggggjhjg},
\er{vjhfhjtjhjhuyyiyGGhhgjgkhhhhhj} and
\er{gjkgjkgjkgjkgjkgjkgjkgjkgjkgjkgjkgjkgjkgjkgjkgjkgjkgjkgjkgjkgjkhjjh}
are completely independent of the vectorial gravitational potential
$\vec v$. Moreover, note that the first two equations in
\er{MaxVacFullPPNffGGvbccgcu8uuuuukjjkhhhihjhjkhhkmgjhjhhkhjhhjjhjhjjhjjjkghhhfjkjkkhjggj}
are completely independent also on the velocity field of the medium
$\vec u$. Furthermore, if $\Omega:=\Omega(t)\subset\mathbb{R}^3$ is
a three-dimensional domain, moving with velocity field $\vec u$
together with the given medium, then, by the third equation of
either
\er{MaxVacFullPPNffGGvbccgcu8uuuuukjjkhhhihjhjkhhkmgjhjhhkhjhhjjhjhjjhjjjkghhhfhjijhjhkhhk}
or
\er{MaxVacFullPPNffGGvbccgcu8uuuuukjjkhhhihjhjkhhkmgjhjhhkhjhhjjhjhjjhjjjkghhhfjkjkkhjggj},
using part {{\bf (iii)}} of Proposition
\ref{jgjjjjjjjjjjjjjjjjjjjjjjjjjjjjjjjjjkhhkh} and the Divergence
Integral Theorem of the Calculus we deduce:
\begin{multline}\label{hjgghffhhffgfgfghkhjhjggjkhkhbnv}
\frac{d}{dt}\left(\iiint\tilde\rho^*\,
d\Omega(t)\right)=\iiint\left(\frac{\partial \tilde\rho^*}{\partial
t}+{div}_{\vec x}\left\{\tilde\rho^* \vec u\right\}\right)\,
d\Omega(t)\approx\\-\iiint\left({div}_{\vec x}\left\{\vec {\tilde
j^*}_{cond}\right\}\right)\, d\Omega(t)=-\iint\vec {\tilde
j^*}_{cond}\,\cdot\vec n\, d\left(\partial\Omega(t)\right),
\end{multline}
where $\partial\Omega(t)$ is the two dimensional boundary of
$\Omega(t)$, oriented by outer unit normal $\vec n$. Moreover,
analogously, by the second equation in
\er{uiyuyyyyyyyyyyyyyyyyyyyyyyyyyyyyyyyyyyyyyyyyyyyyyyykljggngn} we
deduce:
\begin{multline}\label{hjgghffhhffgfgfghkhjhjggjkhkhjhgjgj}
\frac{d}{dt}\left(\iiint\rho\,
d\Omega(t)\right)=\iiint\left(\frac{\partial \rho}{\partial
t}+{div}_{\vec x}\left\{\rho \vec u\right\}\right)\,
d\Omega(t)=\\-\iiint\left({div}_{\vec x}\left\{\vec
{j}_{cond}\right\}\right)\, d\Omega(t)=-\iint\vec
{j}_{cond}\,\cdot\vec n\, d\left(\partial\Omega(t)\right).
\end{multline}
So, we have
\begin{equation}\label{hjgghffhhffgfgfghkhjhjggjkhkhbnv44jhggggggggggghkjh}
\frac{d}{dt}\left(\iiint\rho\, d\Omega(t)\right)=-\iint\vec
{j}_{cond}\,\cdot\vec n\, d\left(\partial\Omega(t)\right)=-\iint\vec
{\tilde j^*}_{cond}\,\cdot\vec n\,
d\left(\partial\Omega(t)\right)\approx
\frac{d}{dt}\left(\iiint\tilde\rho^*\, d\Omega(t)\right).
\end{equation}
Thus, denoting
\begin{equation}\label{khjjhjhgjgjggj}
\begin{cases}
Q_{\Omega(t)}(t):=\iiint\rho\, d\Omega(t),\quad \tilde Q^*_{\Omega(t)}(t):=\iiint\tilde\rho^*\, d\Omega(t)\quad\text{and}\quad\\
I_{\partial\Omega(t)}(t):=\iint\vec {j}_{cond}\,\cdot\vec n
=\iint\vec {\tilde j^*}_{cond}\,\cdot\vec n\,
d\left(\partial\Omega(t)\right):=\tilde I^*_{\partial\Omega(t)}(t),
\end{cases}
\end{equation}
which are the total charge inside the domain $\Omega(t)$, the total
charge$^*$ inside the domain $\Omega(t)$ and the corresponding total
current outward the boundary of $\Omega(t)$, by
\er{hjgghffhhffgfgfghkhjhjggjkhkhbnv44jhggggggggggghkjh} we have
\begin{equation}\label{hjgghffhhffgfgfghkhjhjggjkhkhbnv44jhggggggggggghkjhkjhjhjhjhjh}
\frac{d}{dt}\left(Q_{\Omega(t)}(t)\right)=-I_{\partial\Omega(t)}(t)=-\tilde
I^*_{\partial\Omega(t)}(t)\approx \frac{d}{dt}\left(\tilde
Q^*_{\Omega(t)}(t)\right).
\end{equation}
On the other hand, if $\gamma:=\gamma(t)\subset\mathbb{R}^3$ is a
one-dimensional curve oriented by the unit tangent vector $\vec
t:=\vec t(\vec x,t)$, having the starting and the ending points
$\vec r_{begin}(t), \vec r_{end}(t)$ and moving with velocity field
$\vec u$ together with the given medium, then, by the fourth
equation in either
\er{MaxVacFullPPNffGGvbccgcu8uuuuukjjkhhhihjhjkhhkmgjhjhhkhjhhjjhjhjjhjjjkghhhfhjijhjhkhhk}
or
\er{MaxVacFullPPNffGGvbccgcu8uuuuukjjkhhhihjhjkhhkmgjhjhhkhjhhjjhjhjjhjjjkghhhfjkjkkhjggj}
using part {{\bf (ii)}} of Proposition
\ref{jgjjjjjjjjjjjjjjjjjjjjjjjjjjjjjjjjjkhhkh} we obtain:
\begin{multline*}
\int\vec {\tilde E}\cdot\vec t\, d\gamma(t)\approx-\int\nabla_{\vec
x}\tilde\Psi^*_1\cdot\vec t\,
d\gamma(t)-\frac{1}{c}\int\left(\frac{\partial \vec {\tilde
A^*}}{\partial t}- \vec u\times curl_{\vec x}\vec {\tilde
A^*}+\nabla_{\vec x}\left(\vec u\cdot\vec {\tilde
A^*}\right)\right)\cdot\vec t\,
d\gamma(t)\\=\tilde\Psi^*_1\left(\vec
r_{begin}(t),t\right)-\tilde\Psi^*_1\left(\vec
r_{end}(t),t\right)-\frac{1}{c}\frac{d}{dt}\left(\int\vec {\tilde
A^*}\cdot\vec t\, d\gamma(t)\right),
\end{multline*}
and thus,
\begin{equation}\label{MaxVacFullPPNffGGvbccgcu8uuuuukjjkhhhihjhjkhhkmgjhjhhkhjhhjjhjhjjh11ljkjkljkljjjkjkljkjk11jhgggjh11gjgghhjhhh}
\int\vec {\tilde E}\cdot\vec t\,
d\gamma(t)\approx\tilde\Psi^*_1\left(\vec
r_{begin}(t),t\right)-\tilde\Psi^*_1\left(\vec
r_{end}(t),t\right)-\frac{1}{c}\frac{d}{dt}\left(\int\vec {\tilde
A^*}\cdot\vec t\, d\gamma(t)\right).
\end{equation}
In particular, in the case where $\gamma(t):=\partial\mathcal{S}(t)$
is a boundary of a two-dimensional surface
$\mathcal{S}:=\mathcal{S}(t)\subset\mathbb{R}^3$, oriented by the
unit normal $\vec n:=\vec n(\vec x,t)$, since in the later case we
have $\vec r_{begin}(t)=\vec r_{end}(t)$, using the last equation in
either
\er{MaxVacFullPPNffGGvbccgcu8uuuuukjjkhhhihjhjkhhkmgjhjhhkhjhhjjhjhjjhjjjkghhhfhjijhjhkhhk}
or
\er{MaxVacFullPPNffGGvbccgcu8uuuuukjjkhhhihjhjkhhkmgjhjhhkhjhhjjhjhjjhjjjkghhhfjkjkkhjggj}
and the Stokes Theorem, we rewrite
\er{MaxVacFullPPNffGGvbccgcu8uuuuukjjkhhhihjhjkhhkmgjhjhhkhjhhjjhjhjjh11ljkjkljkljjjkjkljkjk11jhgggjh11gjgghhjhhh}
as
\begin{equation}\label{MaxVacFullPPNffGGvbccgcu8uuuuukjjkhhhihjhjkhhkmgjhjhhkhjhhjjhjhjjh11ljkjkljkljjjkjkljkjk11jhgggjh11gjgghhjhhhhjjhjghjjhhjhjh}
\int\vec {\tilde E}\cdot\vec t\,
d\gamma(t)\approx-\frac{1}{c}\frac{d}{dt}\left(\iint\vec {\tilde
B}\cdot\vec n\,
d\mathcal{S}(t)\right)=-\frac{1}{c}\frac{d}{dt}\left(\iint
curl_{\vec x}\vec {\tilde A^*}\cdot\vec n\, d\mathcal{S}(t)\right).
\end{equation}
Finally, the obtained relations in \er{MaxVagjkgggggggggggggjhjg},
\er{vjhfhjtjhjhuyyiyGGhhgjgkhhhhhj}, \er{vjhiiiihhjjhjhkkk},
\er{hjgghffhhffgfgfghkhjhjggjkhkhbnv44jhggggggggggghkjhkjhjhjhjhjh}
with \er{khjjhjhgjgjggj},
\er{MaxVacFullPPNffGGvbccgcu8uuuuukjjkhhhihjhjkhhkmgjhjhhkhjhhjjhjhjjh11ljkjkljkljjjkjkljkjk11jhgggjh11gjgghhjhhh},
\er{MaxVacFullPPNffGGvbccgcu8uuuuukjjkhhhihjhjkhhkmgjhjhhkhjhhjjhjhjjh11ljkjkljkljjjkjkljkjk11jhgggjh11gjgghhjhhhhjjhjghjjhhjhjh}
and the first two equations in
\er{MaxVacFullPPNffGGvbccgcu8uuuuukjjkhhhihjhjkhhkmgjhjhhkhjhhjjhjhjjhjjjkghhhfjkjkkhjggj}
are completely independent of the velocity field $\vec u$ ( and on
the vectorial gravitational potential $\vec v$). However, the
mentioned relations in \er{MaxVagjkgggggggggggggjhjg},
\er{vjhfhjtjhjhuyyiyGGhhgjgkhhhhhj},
\er{hjgghffhhffgfgfghkhjhjggjkhkhbnv44jhggggggggggghkjhkjhjhjhjhjh}
with \er{khjjhjhgjgjggj},
\er{MaxVacFullPPNffGGvbccgcu8uuuuukjjkhhhihjhjkhhkmgjhjhhkhjhhjjhjhjjh11ljkjkljkljjjkjkljkjk11jhgggjh11gjgghhjhhh},
\er{MaxVacFullPPNffGGvbccgcu8uuuuukjjkhhhihjhjkhhkmgjhjhhkhjhhjjhjhjjh11ljkjkljkljjjkjkljkjk11jhgggjh11gjgghhjhhhhjjhjghjjhhjhjh}
and the first two equations in
\er{MaxVacFullPPNffGGvbccgcu8uuuuukjjkhhhihjhjkhhkmgjhjhhkhjhhjjhjhjjhjjjkghhhfjkjkkhjggj}
are together sufficient for estimation of DC or AC linear chains
and, in particular, they sufficient to obtain the Kirchhoff's
current and voltage laws, Faraday's law of induction and the law of
capacity. Thus, the estimation of DC or quasistationary AC currents
in the case of moving electrical chains and/or in the case of
non-trivial gravitation is completely analogous to that estimation
for the \underline{resting} chains without influence of any
gravitation!

\subsection{Geometric optics inside a moving medium and/or in the presence of gravitational
field}\label{GO}
\subsubsection{Preliminary calculations}
Assume that in some inertial or non-inertial cartesian coordinate
system a \underline{real valued} scalar field $U:=U(\vec x,t)$,
characterizing some wave, satisfies the following wave equation
\begin{equation}\label{MaxVacFullPPNmmmffffffiuiuhjuughbghhjj}
\frac{\partial}{\partial
t}\left\{\frac{1}{c^2_0}\left(\frac{\partial U}{\partial t}+\vec
{\tilde u}\cdot\nabla_{\vec x}U\right)\right\}+\Div_{\vec x}
\left\{\frac{1}{c^2_0}\left(\frac{\partial U}{\partial t}+\vec
{\tilde u}\cdot\nabla_{\vec x} U\right)\vec {\tilde
u}\right\}-\Delta_{\vec x}U=0,
\end{equation}
where $\vec {\tilde u}:=\vec {\tilde u}(\vec x,t)$ is some
moderately changing (in space and in time) speed-like vector field
and $c_0:=c_0(\vec x,t)>0$ is a moderately changing (in space and in
time) scalar quantity, that we call wave propagation speed. Note
that \er{MaxVacFullPPNmmmffffffiuiuhjuughbghhjj} coincides with
\er{MaxVacFullPPNmmmffffffiuiuhjuGGggFGelGHGHGHGG} and thus, in
particular, $U$ can represent one of any scalar components of the
electromagnetic field in the medium without dispersion, i.e. when
\er{EBDHTrans444knbuihuigGG} is valid. In this case $\vec{\tilde u}$
is the linear combination of the velocity of the medium $\vec u$ and
the vectorial gravitational potential $\vec v$. Note also that
\er{MaxVacFull1ninshtrgravortghhghgjkgghklhjgkghghjjkjhjkkggjkhjkhjjhhfhjhkjhjjhkhbbgjlmkmmhzzzyyykkknnnNWBWHWPPNjjlllkkkkkkljkkjjkjkhjjhhgghgggghjjhgjggh}
or
\er{MaxVacFull1ninshtrgravortghhghgjkgghklhjgkghghjjkjhjkkggjkhjkhjjhhfhjhkjhjjhkhbbgjlmkmmhzzzyyykkknnnNWBWHWPPNjjlllkkkkkkljkkjjkjkhjjhhgghgggghjjhgjgghjkjkjk}
also coincide with \er{MaxVacFullPPNmmmffffffiuiuhjuughbghhjj} and
thus, in particular, $U$ could represent the oscillating part of the
pressure $p_1$ in the sound wave. However, in the later case
$\vec{\tilde u}$ is equal to the averaged (macroscopical) velocity
of the fluid/gas $\vec u_0$. Moreover,
\er{MaxVacFull1ninshtrgravortghhghgjkgghklhjgkghghjjkjhjkkggjkhjkhjjhhfhjhkjkhbbgjhzzzyyykkknnnNWBWHWPPNjjjljkjkkhhhhhhhhhhhhfnfffhuyuyuyutjhjhjhggjjfgfhfhjhhgjgjhkjhjhkhkhjj}
and \er{MaxVacFullPPNmmmffffffiuiuhjuGGggFGjjkkjjffjj} also coincide
with \er{MaxVacFullPPNmmmffffffiuiuhjuughbghhjj} and thus, in
particular, $U$ could represent either longitudinal or transverse
wave in an elastic body. In the later case $\vec{\tilde u}$ is also
equal to the averaged (macroscopical) velocity of the elastic medium
$\vec u_0$.

Next if we assume that the fields $\vec {\tilde u}$ and $c_0$ are
independent of the time variable, since  $U$ is a real valued field,
then we can write the field $U$ as a Fourier's Transform on the time
variable:
\begin{equation}\label{MaxVacFullPPNmmmffffffhhtygghGGGGyuhggghghghffgggg}
U(\vec x,t)=Re\left\{2\int_{0}^{+\infty} \hat U(\vec
x,\omega)e^{i\omega t}d\omega\right\}\quad\text{where}\quad\hat
U(\vec x,\omega):=\frac{1}{2\pi}\int_{-\infty}^{+\infty} U(\vec
x,t)e^{-i\omega t}dt\,.
\end{equation}
Moreover, by \er{MaxVacFullPPNmmmffffffiuiuhjuughbghhjj} we obtain
that the Fourier's Transform $\hat U(\vec x,\omega)$ satisfies:
\begin{equation}\label{MaxVacFullPPNmmmffffffiuiuhjuughbghhjjugg}
i\omega\frac{1}{c^2_0}\left(i\omega \hat U+\vec {\tilde
u}\cdot\nabla_{\vec x}\hat U\right)+\Div_{\vec x}
\left\{\frac{1}{c^2_0}\left(i\omega \hat U+\vec {\tilde
u}\cdot\nabla_{\vec x} \hat U\right)\vec {\tilde
u}\right\}-\Delta_{\vec x}\hat U=0.
\end{equation}
Thus by \er{MaxVacFullPPNmmmffffffiuiuhjuughbghhjjugg}, for every
given $\omega$ the monochromatic wave type function
\begin{equation}\label{MaxVacFullPPNmmmffffffhhtygghGGGGyuhggghghghffgggguiu}
U_\omega(\vec x,t):=\hat U(\vec x,\omega)e^{i\omega t}
\end{equation}
is a complex solution of
\begin{equation}\label{MaxVacFullPPNmmmffffffiuiuhjuughbghhjjghghh}
\frac{\partial}{\partial
t}\left\{\frac{1}{c^2_0}\left(\frac{\partial U_\omega}{\partial
t}+\vec {\tilde u}\cdot\nabla_{\vec
x}U_\omega\right)\right\}+\Div_{\vec x}
\left\{\frac{1}{c^2_0}\left(\frac{\partial U_\omega}{\partial
t}+\vec {\tilde u}\cdot\nabla_{\vec x} U_\omega\right)\vec {\tilde
u}\right\}-\Delta_{\vec x}U_\omega=0.
\end{equation}
Note that equation \er{MaxVacFullPPNmmmffffffiuiuhjuughbghhjjghghh}
coincides with \er{MaxVacFullPPNmmmffffffiuiuhjuughbghhjj}.
Moreover, by \er{MaxVacFullPPNmmmffffffhhtygghGGGGyuhggghghghffgggg}
a general real solution of
\er{MaxVacFullPPNmmmffffffiuiuhjuughbghhjj} can be represented as
the real part of a superposition of monochromatic waves of type
$U_\omega=f(\vec x,\omega)e^{i\omega t}$ that satisfy
\er{MaxVacFullPPNmmmffffffiuiuhjuughbghhjjghghh} for every $\omega$.
Finally if we consider a complex valued function $\mathcal{U}(\vec
x,t)$, defined by
\begin{equation}\label{MaxVacFullPPNmmmffffffhhtygghGGGGyuhggghghghffggggjkjjkjk}
\mathcal{U}(\vec x,t):=2\int_{0}^{+\infty} \hat U(\vec
x,\omega)e^{i\omega t}d\omega,
\end{equation}
then $Re\,\mathcal{U}=U$ and $\mathcal{U}$ is a complex solution of
\er{MaxVacFullPPNmmmffffffiuiuhjuughbghhjj} (i.e. not only the real
part of $\mathcal{U}$ but also the imaginary part solve
\er{MaxVacFullPPNmmmffffffiuiuhjuughbghhjj}).

Next assume that a scalar \underline{complex} field $U:=U(\vec x,t)$
satisfies \er{MaxVacFullPPNmmmffffffiuiuhjuughbghhjj}. In
particular, $U$ can be a monochromatic solution of
\er{MaxVacFullPPNmmmffffffiuiuhjuughbghhjjghghh}. Furthermore, we
represent the complex field $U$ as:
\begin{equation}\label{MaxVacFullPPNmmmffffffiuiuhjuughbghhuiiujjjjjjjj}
U(\vec x,t)=A(\vec x,t)e^{iT(\vec x,t)},
\end{equation}
where $A:=A(\vec x,t)$ is a complex scalar field and $T:=T(\vec
x,t)$ is a real scalar field. Then define
\begin{equation}\label{MaxVacFullPPNmmmffffffiuiuhjuughbghhuiiujjjjjjjjhhhjjj}
\omega:=\left<\,\left|\frac{\partial T}{\partial t}\right|\,\right>,
\end{equation}
where the sign $\left<\cdot\right>$ means the spatial and temporal
averaging. Next define $k_0$ and a scalar field $S:=S(\vec x,t)$ by
\begin{equation}\label{MaxVacFullPPNmmmffffffiuiuhjuughbghhuiiujjjjjjjjhhhjjjkk}
k_0:=\frac{\omega}{c}\quad\quad\text{and}\quad\quad S(\vec
x,t)=\frac{1}{k_0}T(\vec x,t),
\end{equation}
where $c$ is a constant in the Maxwell equations for the vacuum. So
we clearly have
\begin{equation}\label{MaxVacFullPPNmmmffffffiuiuhjuughbghhuiiujjjjjjjjhhhjjjkuuiuijjk}
\left<\,\left|\frac{\partial S}{\partial t}\right|\,\right>=c\,.
\end{equation}
Moreover,
\begin{equation}\label{MaxVacFullPPNmmmffffffiuiuhjuughbghhuiiujj}
U(\vec x,t)=A(\vec x,t)e^{ik_0S(\vec x,t)}\,.
\end{equation}

Then we in position to insert it into the wave equation
\er{MaxVacFullPPNmmmffffffiuiuhjuughbghhjj} of the form
\begin{equation}\label{MaxVacFullPPNmmmffffffiuiuhjuughbghhjjjoojoui}
\left(\frac{\partial}{\partial
t}\left\{\frac{1}{c^2_0}\left(\frac{\partial U}{\partial t}+\vec
{\tilde u}\cdot\nabla_{\vec x}U\right)\right\}+\Div_{\vec x}
\left\{\frac{1}{c^2_0}\left(\frac{\partial U}{\partial t}+\vec
{\tilde u}\cdot\nabla_{\vec x} U\right)\vec {\tilde
u}\right\}\right)-\Delta_{\vec x}U=0.
\end{equation}
Thus, since by \er{MaxVacFullPPNmmmffffffiuiuhjuughbghhuiiujj} we
have
\begin{equation}\label{MaxVacFullPPNmmmffffffiuiuhjuughbghhuiiujjjj}
\begin{cases}
\nabla_{\vec x}U=e^{ik_0S}\left(\nabla_{\vec x}A+ik_0A\nabla_{\vec x}S\right)\\
\frac{\partial U}{\partial t}+\vec {\tilde u}\cdot\nabla_{\vec
x}U=e^{ik_0S}\left(\left(\frac{\partial A}{\partial t}+\vec {\tilde
u}\cdot\nabla_{\vec x}A\right)+ik_0A\left(\frac{\partial S}{\partial
t}+\vec {\tilde u}\cdot\nabla_{\vec x}S\right)\right)\,,
\end{cases}
\end{equation}
and furthermore
\begin{multline}\label{MaxVacFullPPNmmmffffffiuiuhjuughbghhuiiujjjj3}
\Delta_{\vec x}U=e^{ik_0S}\left(\Delta_{\vec x}A+ik_0\nabla_{\vec x}A\cdot\nabla_{\vec x}S+ik_0A\Delta_{\vec x}S\right)+ik_0e^{ik_0S}\nabla_{\vec x}S\cdot\left(\nabla_{\vec x}A+ik_0A\nabla_{\vec x}S\right)\\
=e^{ik_0S}\left(\Delta_{\vec x}A+2ik_0\nabla_{\vec
x}A\cdot\nabla_{\vec x}S+ik_0A\Delta_{\vec
x}S-k^2_0A\left|\nabla_{\vec x}S\right|^2\right)\,,
\end{multline}
and
\begin{multline}\label{MaxVacFullPPNmmmffffffiuiuhjuughbghhuiiujjjj5}
\frac{\partial}{\partial
t}\left\{\frac{1}{c^2_0}\left(\frac{\partial U}{\partial t}+\vec
{\tilde u}\cdot\nabla_{\vec x}U\right)\right\}+\Div_{\vec x}
\left\{\frac{1}{c^2_0}\left(\frac{\partial U}{\partial t}+\vec
{\tilde u}\cdot\nabla_{\vec x} U\right)\vec {\tilde
u}\right\}\\=e^{ik_0S}\left(\frac{\partial}{\partial
t}\left\{\frac{1}{c^2_0}\left(\frac{\partial A}{\partial t}+\vec
{\tilde u}\cdot\nabla_{\vec x}A\right)\right\}+\Div_{\vec x}
\left\{\frac{1}{c^2_0}\left(\frac{\partial A}{\partial t}+\vec
{\tilde u}\cdot\nabla_{\vec x} A\right)\vec {\tilde
u}\right\}\right)\\+ik_0e^{ik_0S}\frac{1}{c^2_0}\left(\frac{\partial
A}{\partial t}+\vec {\tilde u}\cdot\nabla_{\vec
x}A\right)\left(\frac{\partial S}{\partial t}+\vec {\tilde
u}\cdot\nabla_{\vec
x}S\right)\\+ik_0e^{ik_0S}A\left(\frac{\partial}{\partial
t}\left\{\frac{1}{c^2_0}\left(\frac{\partial S}{\partial t}+\vec
{\tilde u}\cdot\nabla_{\vec x}S\right)\right\}+\Div_{\vec x}
\left\{\frac{1}{c^2_0}\left(\frac{\partial S}{\partial t}+\vec
{\tilde u}\cdot\nabla_{\vec x} S\right)\vec {\tilde
u}\right\}\right)\\+ik_0e^{ik_0S}\frac{1}{c^2_0}\left(\frac{\partial
S}{\partial t}+\vec {\tilde u}\cdot\nabla_{\vec
x}S\right)\left(\left(\frac{\partial A}{\partial t}+\vec {\tilde
u}\cdot\nabla_{\vec x}A\right)+ik_0A\left(\frac{\partial S}{\partial
t}+\vec {\tilde u}\cdot\nabla_{\vec x}S\right)\right)
\\=e^{ik_0S}\left(\frac{\partial}{\partial
t}\left\{\frac{1}{c^2_0}\left(\frac{\partial A}{\partial t}+\vec
{\tilde u}\cdot\nabla_{\vec x}A\right)\right\}+\Div_{\vec x}
\left\{\frac{1}{c^2_0}\left(\frac{\partial A}{\partial t}+\vec
{\tilde u}\cdot\nabla_{\vec x} A\right)\vec {\tilde
u}\right\}\right)\\+2ik_0\frac{1}{c^2_0}e^{ik_0S}\left(\frac{\partial
A}{\partial t}+\vec {\tilde u}\cdot\nabla_{\vec
x}A\right)\left(\frac{\partial S}{\partial t}+\vec {\tilde
u}\cdot\nabla_{\vec
x}S\right)\\+ik_0e^{ik_0S}A\left(\frac{\partial}{\partial
t}\left\{\frac{1}{c^2_0}\left(\frac{\partial S}{\partial t}+\vec
{\tilde u}\cdot\nabla_{\vec x}S\right)\right\}+\Div_{\vec x}
\left\{\frac{1}{c^2_0}\left(\frac{\partial S}{\partial t}+\vec
{\tilde u}\cdot\nabla_{\vec x} S\right)\vec {\tilde
u}\right\}\right)\\-k^2_0\frac{1}{c^2_0}e^{ik_0S}A\left(\frac{\partial
S}{\partial t}+\vec {\tilde u}\cdot\nabla_{\vec x}S\right)^2,
\end{multline}
inserting \er{MaxVacFullPPNmmmffffffiuiuhjuughbghhuiiujjjj3} and
\er{MaxVacFullPPNmmmffffffiuiuhjuughbghhuiiujjjj5} into
\er{MaxVacFullPPNmmmffffffiuiuhjuughbghhjjjoojoui} we deduce:
\begin{multline}\label{MaxVacFullPPNmmmffffffiuiuhjuughbghhiuijghghghhjhjhjjuiuiui}
k^2_0\left(\left|\nabla_{\vec
x}S\right|^2-\frac{1}{c^2_0}\left(\frac{\partial S}{\partial t}+\vec
{\tilde u}\cdot\nabla_\vec x
S\right)^2\right)A\\+\left(\frac{\partial}{\partial
t}\left\{\frac{1}{c^2_0}\left(\frac{\partial A}{\partial t}+\vec
{\tilde u}\cdot\nabla_{\vec x}A\right)\right\}+\Div_{\vec x}
\left\{\frac{1}{c^2_0}\left(\frac{\partial A}{\partial t}+\vec
{\tilde u}\cdot\nabla_{\vec x} A\right)\vec {\tilde
u}\right\}\right)-\Delta_{\vec
x}A\\
+i\kappa_0 A\left(\frac{\partial}{\partial
t}\left\{\frac{1}{c^2_0}\left(\frac{\partial S}{\partial t}+\vec
{\tilde u}\cdot\nabla_{\vec x}S\right)\right\}+\Div_{\vec x}
\left\{\frac{1}{c^2_0}\left(\frac{\partial S}{\partial t}+\vec
{\tilde u}\cdot\nabla_{\vec x} S\right)\vec {\tilde
u}\right\}-\Delta_{\vec
x}S\right)\\+2i\kappa_0\left(\frac{1}{c^2_0}\left(\frac{\partial
S}{\partial t}+\vec {\tilde u}\cdot\nabla_{\vec
x}S\right)\left(\frac{\partial A}{\partial t}+\vec {\tilde
u}\cdot\nabla_{\vec x}A\right)-\nabla_{\vec x}A\cdot\nabla_{\vec
x}S\right)=0.
\end{multline}
Finally, denoting
\begin{equation}\label{MaxVacFullPPNmmmffffffiuiuhjuughbghhiuijghghghhhgjgfghjdfjgddg1133khhujuuuo}
\tilde\omega(\vec x,t):=\kappa_0\left(\frac{\partial S}{\partial
t}\left(\vec x,t\right)+\vec {\tilde u}\left(\vec
x,t\right)\cdot\nabla_{\vec x}S\left(\vec x,t\right)\right),
\end{equation}
\begin{equation}\label{MaxVacFullPPNmmmffffffiuiuhjuughbghhiuijghghghhhgjgfghjdfjgddg113355}
\vec k(\vec x,t):=c_0(\vec x,t)\left(\frac{\partial S}{\partial
t}(\vec x,t)+\vec {\tilde u}(\vec x,t)\cdot\nabla_{\vec x}S(\vec
x,t)\right)^{-1}\nabla_{\vec x}S(\vec x,t),
\end{equation}
\begin{equation}\label{MaxVacFullPPNmmmffffffiuiuhjuughbghhiuijghghghhhgjgfghjdfjgddg1133}
\vec h(\vec x,t):=\vec {\tilde u}(\vec x,t)-c^2_0(\vec
x,t)\left(\frac{\partial S}{\partial t}(\vec x,t)+\vec {\tilde
u}(\vec x,t)\cdot\nabla_{\vec x}S(\vec x,t)\right)^{-1}\nabla_{\vec
x}S(\vec x,t)=\vec {\tilde u}(\vec x,t)-c_0(\vec x,t)\vec k(\vec
x,t),
\end{equation}
and
\begin{multline}\label{MaxVacFullPPNmmmffffffiuiuhjuughbghhiuijghghghhhfhhghghguygtjuuujjjkhjhjh11kklkloiouu33}
G\left(\vec x,t\right):=\\
\left(\frac{c^2_0}{2\left(\frac{\partial S}{\partial t}+\vec {\tilde
u}\cdot\nabla_{\vec x}S\right)}\left(\Delta_{\vec
x}S-\frac{\partial}{\partial
t}\left\{\frac{1}{c^2_0}\left(\frac{\partial S}{\partial t}+\vec
{\tilde u}\cdot\nabla_{\vec x}S\right)\right\}-\Div_{\vec x}
\left\{\frac{1}{c^2_0}\left(\frac{\partial S}{\partial t}+\vec
{\tilde u}\cdot\nabla_{\vec x} S\right)\vec {\tilde
u}\right\}\right)\right)(\vec x,t)
\\=
\left(\frac{c^2_0}{2\tilde\omega}\left(\Div_{\vec
x}\left\{\frac{\tilde\omega}{c_0}\,\vec
k\right\}-\frac{\partial}{\partial
t}\left\{\frac{\tilde\omega}{c^2_0}\right\}-\Div_{\vec x}
\left\{\frac{\tilde\omega}{c^2_0}\,\vec {\tilde
u}\right\}\right)\right)(\vec x,t)\\=-
\left(\frac{c^2_0}{2\tilde\omega}\left(\frac{\partial}{\partial
t}\left\{\frac{\tilde\omega}{c^2_0}\right\}+\vec h\cdot\nabla_{\vec
x} \left\{\frac{\tilde\omega}{c^2_0}\,\right\}\right)\right)(\vec
x,t)-\frac{1}{2}\left(\Div_{\vec x}\vec h(\vec x,t)\right),
\end{multline}
by inserting
\er{MaxVacFullPPNmmmffffffiuiuhjuughbghhiuijghghghhhgjgfghjdfjgddg1133khhujuuuo},
\er{MaxVacFullPPNmmmffffffiuiuhjuughbghhiuijghghghhhgjgfghjdfjgddg113355},
\er{MaxVacFullPPNmmmffffffiuiuhjuughbghhiuijghghghhhgjgfghjdfjgddg1133}
and
\er{MaxVacFullPPNmmmffffffiuiuhjuughbghhiuijghghghhhfhhghghguygtjuuujjjkhjhjh11kklkloiouu33}
into
\er{MaxVacFullPPNmmmffffffiuiuhjuughbghhiuijghghghhjhjhjjuiuiui} we
clearly have:
\begin{multline}\label{MaxVacFullPPNmmmffffffiuiuhjuughbghhiuijghghghhhfhhghghguygtjjjkkjjkkhhklklkjk33hhjj}
\frac{\tilde\omega}{2}\left(|\vec
k|^2-1\right)A+\frac{c^2_0}{2\tilde\omega}\left\{\frac{\partial}{\partial
t}\left\{\frac{1}{c^2_0}\left(\frac{\partial A}{\partial t}+\vec
{\tilde u}\cdot\nabla_{\vec x}A\right)\right\}+\Div_{\vec x}
\left\{\frac{1}{c^2_0}\left(\frac{\partial A}{\partial t}+\vec
{\tilde u}\cdot\nabla_{\vec x} A\right)\vec {\tilde
u}\right\}-\Delta_{\vec x}A\right\}\\+i\left(\frac{\partial
A}{\partial t}
+\vec h\cdot\nabla_{\vec x}A
-G\, A\right)=0
.
\end{multline}

In particular if we assume for the moment that
\begin{equation}\label{MaxVacFullPPNmmmffffffiuiuhjuughbghhiuijghghghhjhjhhghyuyjjjhhjhjffiiik}
\frac{1}{c^2_0}\left(\frac{\partial S}{\partial t}+\vec {\tilde
u}\cdot\nabla_\vec x S\right)^2=\left|\nabla_\vec x S\right|^2.
\end{equation}
i.e.
\begin{equation}\label{yuyyuyuyu}
|\vec k|^2=1\,,
\end{equation}
then we rewrite
\er{MaxVacFullPPNmmmffffffiuiuhjuughbghhiuijghghghhhfhhghghguygtjjjkkjjkkhhklklkjk33hhjj}
as:
\begin{multline}\label{MaxVacFullPPNmmmffffffiuiuhjuughbghhiuijghghghhhfhhghghguygtjjjkkjjkkhhklklkjk33hhjj111}
\frac{c^2_0}{2\tilde\omega}\left\{\frac{\partial}{\partial
t}\left\{\frac{1}{c^2_0}\left(\frac{\partial A}{\partial t}+\vec
{\tilde u}\cdot\nabla_{\vec x}A\right)\right\}+\Div_{\vec x}
\left\{\frac{1}{c^2_0}\left(\frac{\partial A}{\partial t}+\vec
{\tilde u}\cdot\nabla_{\vec x} A\right)\vec {\tilde
u}\right\}-\Delta_{\vec x}A\right\}\\+i\left(\frac{\partial
A}{\partial t}
+\vec h\cdot\nabla_{\vec x}A
-G\, A\right)=0
.
\end{multline}
Moreover, in the latter case we obviously have
\begin{equation}\label{MaxVacFullPPNmmmffffffiuiuhjuughbghhiuijghghghhhgjgfghjdfjgddg1133khhujuuuo99klklljljljhjhjhyf}
\frac{\tilde\omega^2}{c^2_0k_0^2}=\frac{1}{c^2_0}\left(\frac{\partial
S}{\partial t}+\tilde{\vec u}\cdot\nabla_{\vec
x}S\right)^2=\left|\nabla_{\vec x}S\right|^2, \quad\quad\left|\vec
k\right|^2=1\quad\quad\text{or equivalently}\quad\quad \left|\vec
h-\vec {\tilde u}\right|^2=c^2_0\,,
\end{equation}
\begin{equation}\label{MaxVacFullPPNmmmffffffiuiuhjuughbghhiuijghghghhhgjgfghjdfjgddg1133khhujuuuo33uuu}
\frac{\partial S}{\partial t}+\vec {\tilde u}\cdot\nabla_{\vec
x}S=c_0\vec k\cdot\nabla_{\vec x}S\quad\quad \text{or
equivalently}\quad\quad\quad\quad \frac{\partial S}{\partial t}+\vec
h\cdot\nabla_{\vec x}S=0\,,
\end{equation}
\begin{equation}\label{MaxVacFullPPNmmmffffffiuiuhjuughbghhiuijghghghhhgjgfghjdfjgddg1133khhujuuuo99klklljljljhjhjhyfjkkjkhkiili}
\nabla_{\vec x}S\,=\,\left(\vec k(\vec x,t)\cdot\nabla_{\vec
x}S\right)\,\vec k\,,
\end{equation}
\begin{equation}\label{MaxVacFullPPNmmmffffffiuiuhjuughbghhiuijghghghhhgjgfghjdfjgddg11335533jkjjkujjj}
\frac{c_0}{\tilde\omega}\vec k=\frac{1}{k_0}\left|\nabla_{\vec
x}S\right|^{-2}\nabla_{\vec
x}S=\frac{c^2_0k_0}{\tilde\omega^2}\nabla_{\vec x}S,
\end{equation}
and
\begin{equation}\label{MaxVacFullPPNmmmffffffiuiuhjuughbghhiuijghghghhhgjgfghjdfjgddg1133khhujuuuo33uuuiujkjk}
\frac{1}{k_0}\,\nabla_{\vec x}\tilde\omega=\frac{\partial}{\partial
t}\left\{\nabla_{\vec x}S\right\}+\nabla^2_{\vec x}S\cdot\vec
{\tilde u}+\left\{d_{\vec x}\vec {\tilde
u}\right\}^T\cdot\nabla_{\vec x}S\,.
\end{equation}
In particular,
\begin{equation}\label{MaxVacFullPPNmmmffffffiuiuhjuughbghhiuijghghghhhgjgfghjdfjgddg1133khhujuuuo33uuuiujkjkuyuy}
\frac{2}{k_0}\,\nabla_{\vec x}S\cdot\nabla_{\vec
x}\tilde\omega=\frac{\partial}{\partial t}\left\{\left|\nabla_{\vec
x}S\right|^2\right\}+\vec {\tilde u}\cdot\nabla_{\vec
x}\left\{\left|\nabla_{\vec x}S\right|^2\right\}+\left(\nabla_{\vec
x}S\cdot\left(\left\{d_{\vec x}\vec {\tilde u}\right\}^T+d_{\vec
x}\vec {\tilde u}\right)\right)\cdot\nabla_{\vec x}S\,,
\end{equation}
so that
\begin{equation}\label{MaxVacFullPPNmmmffffffiuiuhjuughbghhiuijghghghhhgjgfghjdfjgddg1133khhujuuuo33uuuiujkjkuyuyiio}
c_0\,\frac{1}{c^2_0}\,\vec k\cdot\nabla_{\vec
x}\left\{\tilde\omega^2\right\}=\frac{\partial}{\partial
t}\left\{\frac{\tilde\omega^2}{c^2_0}\right\}+\vec {\tilde
u}\cdot\nabla_{\vec
x}\left\{\frac{\tilde\omega^2}{c^2_0}\right\}+\frac{\tilde\omega^2}{c^2_0}\left(\vec
k\cdot\left(\left\{d_{\vec x}\vec {\tilde u}\right\}^T+d_{\vec
x}\vec {\tilde u}\right)\right)\cdot\vec k\,.
\end{equation}
I.e.
\begin{equation}\label{MaxVacFullPPNmmmffffffiuiuhjuughbghhiuijghghghhhgjgfghjdfjgddg1133khhujuuuo33uuuiujkjkuyuyiioiiuui}
\left(\frac{\tilde\omega^2}{c^2_0}\right)^{-1}\left(\frac{\partial}{\partial
t}\left\{\frac{\tilde\omega^2}{c^2_0}\right\}+\vec
h\cdot\nabla_{\vec
x}\left\{\frac{\tilde\omega^2}{c^2_0}\right\}\right)=
\frac{1}{c^2_0}\left(c_0\vec k\right)\cdot\nabla_{\vec
x}\left\{c^2_0\right\}-\left(\vec k\cdot\left(\left\{d_{\vec x}\vec
{\tilde u}\right\}^T+d_{\vec x}\vec {\tilde
u}\right)\right)\cdot\vec k\,.
\end{equation}
We rewrite it as:
\begin{equation}\label{MaxVacFullPPNmmmffffffiuiuhjuughbghhiuijghghghhhgjgfghjdfjgddg1133khhujuuuo33uuuiujkjkuyuyiioiiuuiiiiu}
\frac{1}{\tilde\omega}\left(\frac{\partial \tilde\omega}{\partial
t}+\vec h\cdot\nabla_{\vec
x}\tilde\omega\right)=\frac{1}{c_0}\left(\frac{\partial
c_0}{\partial t}+\vec {\tilde u}\cdot\nabla_{\vec
x}c_0\right)-\frac{1}{2}\left(\vec k\cdot\left(\left\{d_{\vec x}\vec
{\tilde u}\right\}^T+d_{\vec x}\vec {\tilde
u}\right)\right)\cdot\vec k\,,
\end{equation}
provided, we have
\er{MaxVacFullPPNmmmffffffiuiuhjuughbghhiuijghghghhjhjhhghyuyjjjhhjhjffiiik}.

\subsubsection{Derivation of the Eikonal equation}\label{ekGO}
In particular, in the case both $S$ and $A$ are real, since the zero
complex number has both real and imaginary part equal to zero, by
\er{MaxVacFullPPNmmmffffffiuiuhjuughbghhiuijghghghhjhjhjjuiuiui} we
have:
\begin{multline}\label{MaxVacFullPPNmmmffffffiuiuhjuughbghhiuijghghghhjhjhjj}
k^2_0\left(\left|\nabla_{\vec
x}S\right|^2-\frac{1}{c^2_0}\left(\frac{\partial S}{\partial t}+\vec
{\tilde u}\cdot\nabla_\vec x
S\right)^2\right)A\\+\frac{\partial}{\partial
t}\left\{\frac{1}{c^2_0}\left(\frac{\partial A}{\partial t}+\vec
{\tilde u}\cdot\nabla_{\vec x}A\right)\right\}+\Div_{\vec x}
\left\{\frac{1}{c^2_0}\left(\frac{\partial A}{\partial t}+\vec
{\tilde u}\cdot\nabla_{\vec x} A\right)\vec {\tilde
u}\right\}-\Delta_{\vec x}A=0,
\end{multline}
and
\begin{multline}\label{MaxVacFullPPNmmmffffffiuiuhjuughbghhiuijghghghhhfhjj}
A\left(\frac{\partial}{\partial
t}\left\{\frac{1}{c^2_0}\left(\frac{\partial S}{\partial t}+\vec
{\tilde u}\cdot\nabla_{\vec x}S\right)\right\}+\Div_{\vec x}
\left\{\frac{1}{c^2_0}\left(\frac{\partial S}{\partial t}+\vec
{\tilde u}\cdot\nabla_{\vec x} S\right)\vec {\tilde
u}\right\}-\Delta_{\vec
x}S\right)\\+2\left(\frac{1}{c^2_0}\left(\frac{\partial S}{\partial
t}+\vec {\tilde u}\cdot\nabla_{\vec x}S\right)\left(\frac{\partial
A}{\partial t}+\vec {\tilde u}\cdot\nabla_{\vec
x}A\right)-\nabla_{\vec x}A\cdot\nabla_{\vec x}S\right)=0.
\end{multline}
Next assume either the delicate or the rough Geometric Optics
approximation that are good for the electromagnetic wave of high
frequency for example for the visible light. The delicate Geometric
Optics approximation means the following: assume that the changes in
time of $c_0$, $\vec {\tilde u}$, $A$ and $S$ become essential after
certain interval of time $T_e$ and the changes in space of $c_0$,
$A$ and $S$ become essential in the spatial landscape $L_e$. Then we
assume that
\begin{equation}\label{slochangGGaaffgfgjhjggh}
k^2_0c^2_0T^2_e\,\gg\, 1\quad\text{and}\quad k^2_0L^2_e\,\gg\, 1.
\end{equation}
On the other hand, the rough Geometric Optics approximation
(stronger than \er{slochangGGaaffgfgjhjggh}) means the following:
\begin{equation}\label{slochangGGaaffgfgjhjggh11}
k_0c_0T_e\,\gg\, 1\quad\text{and}\quad k_0L_e\,\gg\, 1.
\end{equation}
Moreover, we assume that the order of $c_0$ is less or equal to the
order of $c$. In particular, estimation \er{slochangGGaaffgfgjhjggh}
implies
\begin{multline}\label{MaxVacFullPPNmmmffffffhhtygghGGGGyuhggghghghffgiooiiojhjjhhjuiuiioiooikjiij11}
\frac{\left| d_{\vec x}\vec {\tilde u}+\left\{d_{\vec x}\vec {\tilde
u}\right\}^T\right|^2}{
c^2_0}+\frac{\left(| \nabla_{\vec
x}A|+\frac{1}{c_0}\left|\frac{\partial A}{\partial t}+\vec {\tilde
u}\cdot\nabla_{\vec x}A\right|\right)^2}{|A|^2}+\frac{\left(|
\nabla_{\vec x}c_0|+\frac{1}{c_0}\left|\frac{\partial c_0}{\partial
t}+\vec {\tilde u}\cdot\nabla_{\vec
x}c_0\right|\right)^2}{|c_0|^2}\\+\frac{ \left| \nabla^2_{\vec
x}S\right|^2+ \frac{1}{c^2_0}\left|\nabla_{\vec
x}\left(\frac{\partial S}{\partial t}+\vec {\tilde
u}\cdot\nabla_{\vec x}S\right)\right|^2 }{ \left| \nabla_{\vec
x}S\right|^2} \ll k^2_0 \quad\quad\text{and}\quad\quad
\frac{\left|\frac{\partial S}{\partial t}+\vec {\tilde
u}\cdot\nabla_{\vec x}S-c\right|^2}{c^2}\ll 1\,,
\end{multline}

Thus, using \er{slochangGGaaffgfgjhjggh}, we approximate
\er{MaxVacFullPPNmmmffffffiuiuhjuughbghhiuijghghghhjhjhjj} as:
\begin{equation}\label{MaxVacFullPPNmmmffffffiuiuhjuughbghhiuijghghghhjhjhhghyuyjjjhhjhjff}
\frac{1}{c^2_0}\left(\frac{\partial S}{\partial t}+\vec {\tilde
u}\cdot\nabla_\vec x S\right)^2=\left|\nabla_\vec x S\right|^2.
\end{equation}
Moreover, without any use in either \er{slochangGGaaffgfgjhjggh} or
\er{slochangGGaaffgfgjhjggh11} we rewrite
\er{MaxVacFullPPNmmmffffffiuiuhjuughbghhiuijghghghhhfhjj} as:
\begin{multline}\label{MaxVacFullPPNmmmffffffiuiuhjuughbghhiuijghghghhhfhhghghguygtjjjkkjjkkhh}
\left(\frac{\partial}{\partial
t}\left\{\frac{1}{c^2_0}\left(\frac{\partial S}{\partial t}+\vec
{\tilde u}\cdot\nabla_{\vec x}S\right)\right\}+\Div_{\vec x}
\left\{\frac{1}{c^2_0}\left(\frac{\partial S}{\partial t}+\vec
{\tilde u}\cdot\nabla_{\vec x} S\right)\vec {\tilde
u}\right\}-\left(\Delta_{\vec x}S\right)\right)A
\\+\frac{2}{c^2_0}\left(\frac{\partial S}{\partial
t}+\vec {\tilde u}\cdot\nabla_{\vec x}S\right)\left(\frac{\partial
A}{\partial t}+\vec {\tilde u}\cdot\nabla_{\vec x}A\right)
-2\nabla_{\vec x}S\cdot\nabla_{\vec x}A=0\,,
\end{multline}
Equality
\er{MaxVacFullPPNmmmffffffiuiuhjuughbghhiuijghghghhjhjhhghyuyjjjhhjhjff}
is called the Eikonal equation and equality
\er{MaxVacFullPPNmmmffffffiuiuhjuughbghhiuijghghghhhfhhghghguygtjjjkkjjkkhh}
is called the equation of the beam propagation. Then, as before, we
deduce that equation
\er{MaxVacFullPPNmmmffffffiuiuhjuughbghhiuijghghghhjhjhhghyuyjjjhhjhjff}
is invariant under the change of non-inertial cartesian coordinate
system, provided that under such change we have $S'=S$. Moreover,
\er{MaxVacFullPPNmmmffffffiuiuhjuughbghhiuijghghghhhfhhghghguygtjjjkkjjkkhh}
is also invariant under the change of non-inertial cartesian
coordinate system, in the case that under such change we have
$A'=A$, provided that $S'=S$.
%
%
%
%
%
%
Furthermore, note that although we established Eikonal equation
\er{MaxVacFullPPNmmmffffffiuiuhjuughbghhiuijghghghhjhjhhghyuyjjjhhjhjff}
in the delicate Geometric Optics approximation
\er{slochangGGaaffgfgjhjggh}, since by
\er{MaxVacFullPPNmmmffffffiuiuhjuughbghhuiiujj} we have
\begin{equation}\label{MaxVacFullPPNmmmffffffiuiuhjuughbghhuiiujjjjjjjjgjjg}
U(\vec x,t)=A(\vec x,t)e^{ik_0S(\vec x,t)}\,,
\end{equation}
then, considering the approximations of $A$ and $S$ as in
\er{MaxVacFullPPNmmmffffffiuiuhjuughbghhiuijghghghhhfhhghghguygtjjjkkjjkkhh},
\er{MaxVacFullPPNmmmffffffiuiuhjuughbghhiuijghghghhjhjhhghyuyjjjhhjhjff}
and putting them into
\er{MaxVacFullPPNmmmffffffiuiuhjuughbghhuiiujjjjjjjjgjjg} we
actually establish $U(\vec x,t)$ only in the \underline{rough}
Geometric Optics approximation  \er{slochangGGaaffgfgjhjggh11} .

Finally, as before, denoting
\begin{equation}\label{MaxVacFullPPNmmmffffffiuiuhjuughbghhiuijghghghhhgjgfghjdfjgddg1133khhujuuuo99}
\tilde\omega(\vec x,t):=\kappa_0\left(\frac{\partial S}{\partial
t}\left(\vec x,t\right)+\vec {\tilde u}\left(\vec
x,t\right)\cdot\nabla_{\vec x}S\left(\vec x,t\right)\right),
\end{equation}
\begin{equation}\label{MaxVacFullPPNmmmffffffiuiuhjuughbghhiuijghghghhhgjgfghjdfjgddg11335599}
\vec k(\vec x,t):=c_0(\vec x,t)\left(\frac{\partial S}{\partial
t}(\vec x,t)+\vec {\tilde u}(\vec x,t)\cdot\nabla_{\vec x}S(\vec
x,t)\right)^{-1}\nabla_{\vec x}S(\vec x,t),
\end{equation}
\begin{equation}\label{MaxVacFullPPNmmmffffffiuiuhjuughbghhiuijghghghhhgjgfghjdfjgddg113399}
\vec h(\vec x,t):=\vec {\tilde u}(\vec x,t)-c^2_0(\vec
x,t)\left(\frac{\partial S}{\partial t}(\vec x,t)+\vec {\tilde
u}(\vec x,t)\cdot\nabla_{\vec x}S(\vec x,t)\right)^{-1}\nabla_{\vec
x}S(\vec x,t)=\vec {\tilde u}(\vec x,t)-c_0(\vec x,t)\vec k(\vec
x,t),
\end{equation}
and
\begin{multline}\label{MaxVacFullPPNmmmffffffiuiuhjuughbghhiuijghghghhhfhhghghguygtjuuujjjkhjhjh11kklkloiouu3399}
G\left(\vec x,t\right):=\\
\left(\frac{c^2_0}{2\left(\frac{\partial S}{\partial t}+\vec {\tilde
u}\cdot\nabla_{\vec x}S\right)}\left(\Delta_{\vec
x}S-\frac{\partial}{\partial
t}\left\{\frac{1}{c^2_0}\left(\frac{\partial S}{\partial t}+\vec
{\tilde u}\cdot\nabla_{\vec x}S\right)\right\}-\Div_{\vec x}
\left\{\frac{1}{c^2_0}\left(\frac{\partial S}{\partial t}+\vec
{\tilde u}\cdot\nabla_{\vec x} S\right)\vec {\tilde
u}\right\}\right)\right)(\vec x,t)
\\=
\left(\frac{c^2_0}{2\tilde\omega}\left(\Div_{\vec
x}\left\{\frac{\tilde\omega}{c_0}\,\vec
k\right\}-\frac{\partial}{\partial
t}\left\{\frac{\tilde\omega}{c^2_0}\right\}-\Div_{\vec x}
\left\{\frac{\tilde\omega}{c^2_0}\,\vec {\tilde
u}\right\}\right)\right)(\vec x,t)\\=-
\left(\frac{c^2_0}{2\tilde\omega}\left(\frac{\partial}{\partial
t}\left\{\frac{\tilde\omega}{c^2_0}\right\}+\vec h\cdot\nabla_{\vec
x} \left\{\frac{\tilde\omega}{c^2_0}\,\right\}\right)\right)(\vec
x,t)-\frac{1}{2}\left(\Div_{\vec x}\vec h(\vec x,t)\right),
\end{multline}
we rewrite
\er{MaxVacFullPPNmmmffffffiuiuhjuughbghhiuijghghghhhfhhghghguygtjjjkkjjkkhh}
and
\er{MaxVacFullPPNmmmffffffiuiuhjuughbghhiuijghghghhjhjhhghyuyjjjhhjhjff}
as:
\begin{equation}\label{MaxVacFullPPNmmmffffffiuiuhjuughbghhiuijghghghhhfhhghghguygtjjjkkjjkkhhklklkjk33hhjjjlj}
\frac{\partial A}{\partial t}(\vec x,t)
+\vec h(\vec x,t)\cdot\nabla_{\vec x}A(\vec x,t)
-G(\vec x,t) A(\vec x,t)=0\,,
\end{equation}
and
\begin{equation}\label{MaxVacFullPPNmmmffffffiuiuhjuughbghhiuijghghghhhgjgfghjdfjgddg1133khhujuuuo99klklljljlj}
\left|\vec k(\vec x,t)\right|^2=1\quad\quad\text{or
equivalently}\quad\quad \left|\vec h(\vec x,t)-\vec {\tilde u}(\vec
x,t)\right|^2=c^2_0(\vec x,t)\,.
\end{equation}

Moreover, as before, we obviously have
\begin{equation}\label{MaxVacFullPPNmmmffffffiuiuhjuughbghhiuijghghghhhgjgfghjdfjgddg1133khhujuuuo99klklljljljhjhjhyf3}
\frac{\tilde\omega^2}{c^2_0k_0^2}=\frac{1}{c^2_0}\left(\frac{\partial
S}{\partial t}+\tilde{\vec u}\cdot\nabla_{\vec
x}S\right)^2=\left|\nabla_{\vec x}S\right|^2,\quad\quad \left|\vec
k\right|^2=1\quad\quad\text{or equivalently}\quad\quad \left|\vec
h-\vec {\tilde u}\right|^2=c^2_0\,,
\end{equation}
\begin{equation}\label{MaxVacFullPPNmmmffffffiuiuhjuughbghhiuijghghghhhgjgfghjdfjgddg1133khhujuuuo33uuu3}
\frac{\partial S}{\partial t}+\vec {\tilde u}\cdot\nabla_{\vec
x}S=c_0\vec k\cdot\nabla_{\vec x}S\quad\quad \text{or
equivalently}\quad\quad\quad\quad \frac{\partial S}{\partial t}+\vec
h\cdot\nabla_{\vec x}S=0\,,
\end{equation}
\begin{equation}\label{MaxVacFullPPNmmmffffffiuiuhjuughbghhiuijghghghhhgjgfghjdfjgddg1133khhujuuuo99klklljljljhjhjhyfjkkjkhkiili3}
\nabla_{\vec x}S\,=\,\left(\vec k(\vec x,t)\cdot\nabla_{\vec
x}S\right)\,\vec k\,,
\end{equation}
\begin{equation}\label{MaxVacFullPPNmmmffffffiuiuhjuughbghhiuijghghghhhgjgfghjdfjgddg11335533jkjjkujjj3}
\frac{c_0}{\tilde\omega}\vec k=\frac{1}{k_0}\left|\nabla_{\vec
x}S\right|^{-2}\nabla_{\vec
x}S=\frac{c^2_0k_0}{\tilde\omega^2}\nabla_{\vec x}S,
\end{equation}
\begin{equation}\label{MaxVacFullPPNmmmffffffiuiuhjuughbghhiuijghghghhhgjgfghjdfjgddg1133khhujuuuo33uuuiujkjk3}
\frac{1}{k_0}\,\nabla_{\vec x}\tilde\omega=\frac{\partial}{\partial
t}\left\{\nabla_{\vec x}S\right\}+\nabla^2_{\vec x}S\cdot\vec
{\tilde u}+\left\{d_{\vec x}\vec {\tilde
u}\right\}^T\cdot\nabla_{\vec x}S\,,
\end{equation}
and
\begin{equation}\label{MaxVacFullPPNmmmffffffiuiuhjuughbghhiuijghghghhhgjgfghjdfjgddg1133khhujuuuo33uuuiujkjkuyuyiioiiuuiiiiu3}
\frac{1}{\tilde\omega}\left(\frac{\partial \tilde\omega}{\partial
t}+\vec h\cdot\nabla_{\vec
x}\tilde\omega\right)=\frac{1}{c_0}\left(\frac{\partial
c_0}{\partial t}+\vec {\tilde u}\cdot\nabla_{\vec
x}c_0\right)-\frac{1}{2}\left(\vec k\cdot\left(\left\{d_{\vec x}\vec
{\tilde u}\right\}^T+d_{\vec x}\vec {\tilde
u}\right)\right)\cdot\vec k\,.
\end{equation}

Next if we consider an arbitrary characteristic curve $\vec
r(t):[t_0,b]\to\mathbb{R}^3$ of equation
\er{MaxVacFullPPNmmmffffffiuiuhjuughbghhiuijghghghhhfhhghghguygtjjjkkjjkkhhklklkjk33hhjjjlj},
parameterized by the time variable $t$, defined as a solution of
ordinary differential equation:
\begin{equation}\label{MaxVacFullPPNmmmffffffiuiuhjuughbghhiuijghghghhhfhhghghguygtjuuujjjkk11}
\begin{cases}
\frac{d\vec r}{dt}(t)=\vec h\left(\vec r(t),t\right)\\
\vec r(t_0)=\vec x_0,
\end{cases}
\end{equation}
where $\vec h$ was defined in
\er{MaxVacFullPPNmmmffffffiuiuhjuughbghhiuijghghghhhgjgfghjdfjgddg113399},
then, as before, by
\er{MaxVacFullPPNmmmffffffiuiuhjuughbghhiuijghghghhhfhhghghguygtjjjkkjjkkhhklklkjk33hhjjjlj},
\er{MaxVacFullPPNmmmffffffiuiuhjuughbghhiuijghghghhhfhhghghguygtjuuujjjkk11}
and the Chain rule we have:
\begin{equation}\label{MaxVacFullPPNmmmffffffiuiuhjuughbghhiuijghghghhhfhhghghguygtjuuujjjkhjhjh11}
\frac{d}{dt}\left(A\left(\vec r(t),t\right)\right)=\nabla_{\vec
x}A\left(\vec r(t),t\right)\cdot\frac{d\vec r}{dt}(t)+\frac{\partial
A}{\partial t}\left(\vec r(t),t\right)= G\left(\vec
r(t),t\right)A\left(\vec r(t),t\right),
\end{equation}
where $G$ was defined in
\er{MaxVacFullPPNmmmffffffiuiuhjuughbghhiuijghghghhhfhhghghguygtjuuujjjkhjhjh11kklkloiouu3399}.
Then
\er{MaxVacFullPPNmmmffffffiuiuhjuughbghhiuijghghghhhfhhghghguygtjuuujjjkhjhjh11}
implies
\begin{equation}\label{MaxVacFullPPNmmmffffffiuiuhjuughbghhiuijghghghhhfhhghghguygtjuuujjjkhjhjhffgg11}
A\left(\vec r(t),t\right)=A\left(\vec
x_0,t_0\right)e^{\int_{t_0}^{t}G\left(\vec
r(\tau),\tau\right)d\tau}\quad\quad\forall t\in[t_0,b].
\end{equation}
In particular,
\begin{multline}\label{MaxVacFullPPNmmmffffffiuiuhjuughbghhiuijghghghhhfhhghghguygtjuuujjjkhjhjhffggiouuiiu11}
A\left(\vec x_0,t_0\right)=0\;\;\text{implies}\;\; A\left(\vec
r(t),t\right)=0\quad\forall t\in[t_0,b],\\ \quad\text{and}\quad
A\left(\vec x_0,t_0\right)\neq 0\;\;\text{implies}\;\; A\left(\vec
r(t),t\right)\neq 0\quad\forall t\in[t_0,b].
\end{multline}
Therefore, by
\er{MaxVacFullPPNmmmffffffiuiuhjuughbghhiuijghghghhhfhhghghguygtjuuujjjkhjhjhffggiouuiiu11}
we deduce that the curve that satisfies
\er{MaxVacFullPPNmmmffffffiuiuhjuughbghhiuijghghghhhfhhghghguygtjuuujjjkk11}
coincides with the ray of light that passes through the point $\vec
x_0$ at the instant of time $t_0$. So, equality
\er{MaxVacFullPPNmmmffffffiuiuhjuughbghhiuijghghghhhfhhghghguygtjuuujjjkk11}
is the equation of a ray and the vector field $\vec h$ defined for
every $\vec x$ by
\er{MaxVacFullPPNmmmffffffiuiuhjuughbghhiuijghghghhhgjgfghjdfjgddg113399}
is the direction of the propagation of the ray that passes through
point $\vec x$ at the instant of time $t$. On the other hand, as
before, we can easily prove that the vector field defined in every
inertial or non-inertial coordinate system by
\er{MaxVacFullPPNmmmffffffiuiuhjuughbghhiuijghghghhhgjgfghjdfjgddg113399}
is a \underline{speed-like} vector field. Moreover, by
\er{MaxVacFullPPNmmmffffffiuiuhjuughbghhiuijghghghhhgjgfghjdfjgddg1133khhujuuuo99klklljljljhjhjhyf3}
the following implication holds:
\begin{equation}\label{MaxVacFullPPNmmmffffffiuiuhjuughbghhiuijghghghhhfhhghghguygtjuuujjjkhjhjhffggiouuiiu11mljjjkk;k}
\vec {\tilde u}=0\quad\text{implies}\quad|\vec h|=c_0.
\end{equation}
Finally, by chain rule, for the curve that satisfies
\er{MaxVacFullPPNmmmffffffiuiuhjuughbghhiuijghghghhhfhhghghguygtjuuujjjkk11}
we have:
\begin{equation}\label{MaxVacFullPPNmmmffffffiuiuhjuughbghhiuijghghghhhfhhghghguygtjuuujjjkhjhjhffggiouuiiu11mljjjkk;kljjkljljhj}
\frac{d}{dt}\left(S\left(\vec r(t),t\right)\right)=\frac{\partial
S}{\partial t}\left(\vec r(t),t\right)+\nabla_{\vec x}S\left(\vec
r(t),t\right)\cdot\frac{d\vec r}{dt}(t)=\frac{\partial S}{\partial
t}\left(\vec r(t),t\right)+\nabla_{\vec x}S\left(\vec
r(t),t\right)\cdot\vec h\left(\vec r(t),t\right).
\end{equation}
Thus, by
\er{MaxVacFullPPNmmmffffffiuiuhjuughbghhiuijghghghhhgjgfghjdfjgddg1133khhujuuuo33uuu3}
and
\er{MaxVacFullPPNmmmffffffiuiuhjuughbghhiuijghghghhhfhhghghguygtjuuujjjkhjhjhffggiouuiiu11mljjjkk;kljjkljljhj}
we deduce
\begin{equation}\label{MaxVacFullPPNmmmffffffiuiuhjuughbghhiuijghghghhhfhhghghguygtjuuujjjkhjhjhffggiouuiiu11mljjjkk;kljjkljljhjklklklkl}
\frac{d}{dt}\left(S\left(\vec
r(t),t\right)\right)=0\quad\quad\forall t\geq t_0,
\end{equation}
and so
\begin{equation}\label{MaxVacFullPPNmmmffffffiuiuhjuughbghhiuijghghghhhfhhghghguygtjuuujjjkhjhjhffggiouuiiu11mljjjkk;kljjkljljhjklklklkliio}
S\left(\vec r(t),t\right)=S\left(\vec
x_0,t_0\right)\quad\quad\forall t\geq t_0.
\end{equation}
%
%
%
%
%
%
By all these facts, vector field $\vec h(\vec x,t)$ defined by
\er{MaxVacFullPPNmmmffffffiuiuhjuughbghhiuijghghghhhgjgfghjdfjgddg113399}
can be considered as the vector of the velocity (speed) of the wave
at the point $\vec x$ at the instant of time $t$. Moreover, by
\er{MaxVacFullPPNmmmffffffiuiuhjuughbghhiuijghghghhhgjgfghjdfjgddg1133khhujuuuo99klklljljljhjhjhyf3}
we have
\begin{equation}\label{MaxVacFullPPNmmmffffffiuiuhjuughbghhiuijghghghhjhjhhghyuyjjjhhjhjffkklkljljjkhhjkjhghghghl;kl;;k;kkljjkjkll;l;}
\left|\vec h-\vec{\tilde u}\right|^2=c^2_0.
\end{equation}
\begin{remark}\label{hjgjgghjj11nnjj}
In contrast to the proof of
\er{MaxVacFullPPNmmmffffffiuiuhjuughbghhiuijghghghhjhjhhghyuyjjjhhjhjff},
we do not use any of the Geometric Optics approximations
\er{slochangGGaaffgfgjhjggh} or \er{slochangGGaaffgfgjhjggh11} in
the proof of
\er{MaxVacFullPPNmmmffffffiuiuhjuughbghhiuijghghghhhfhhghghguygtjjjkkjjkkhh}
and
\er{MaxVacFullPPNmmmffffffiuiuhjuughbghhiuijghghghhhfhhghghguygtjjjkkjjkkhhklklkjk33hhjjjlj}.
So,
\er{MaxVacFullPPNmmmffffffiuiuhjuughbghhiuijghghghhhfhhghghguygtjuuujjjkhjhjhffggiouuiiu11}
is still valid without assumption of the Geometric Optics
approximations \er{slochangGGaaffgfgjhjggh} or
\er{slochangGGaaffgfgjhjggh11} and thus, the vector field $\vec h$,
defined by
\er{MaxVacFullPPNmmmffffffiuiuhjuughbghhiuijghghghhhgjgfghjdfjgddg113399},
is the direction of the propagation of the ray that passes through
point $\vec x$ also in the general case. However, without
assumptions of the Geometric Optics approximation we cannot derive
\er{MaxVacFullPPNmmmffffffiuiuhjuughbghhiuijghghghhjhjhhghyuyjjjhhjhjff},
\er{MaxVacFullPPNmmmffffffiuiuhjuughbghhiuijghghghhhfhhghghguygtjuuujjjkhjhjhffggiouuiiu11mljjjkk;k},
\er{MaxVacFullPPNmmmffffffiuiuhjuughbghhiuijghghghhhfhhghghguygtjuuujjjkhjhjhffggiouuiiu11mljjjkk;kljjkljljhjklklklkliio}
and
\er{MaxVacFullPPNmmmffffffiuiuhjuughbghhiuijghghghhjhjhhghyuyjjjhhjhjffkklkljljjkhhjkjhghghghl;kl;;k;kkljjkjkll;l;}
anymore.
\end{remark}

\subsubsection{The case of the monochromatic wave}\label{mnGO}
Next, up to the end of this subsection, consider the case of
monochromatic wave of the constant frequency
$\nu=\frac{\omega}{2\pi}$ where the fields $\vec {\tilde u}$ and
$c_0$ are independent of the time variable i.e. the case of
\er{MaxVacFullPPNmmmffffffiuiuhjuughbghhuiiujjjjjjjj} where we have
\begin{equation}\label{MaxVacFullPPNmmmffffffiuiuhjuughbghhiuijghghghhjhjhhghyuyhhjhjihikjjjjkjljkl}
\begin{cases}
\frac{\partial T}{\partial t}=\omega\\
\frac{\partial A}{\partial t}=0\\
\frac{\partial \vec {\tilde u}}{\partial t}=0\\
\frac{\partial c_0}{\partial t}=0.
\end{cases}
\end{equation}
Moreover, we assume the rough Geometric Optics approximation
\er{slochangGGaaffgfgjhjggh11}. We also assume that either our
medium has no dispersion (i.e. in the case of an electromagnetic
wave \er{EBDHTrans444knbuihuigGG} is valid), or the velocity of the
medium is negligible i.e. $\vec u\equiv 0$ (and so in the case of an
electromagnetic wave we have $\vec{\tilde u}\equiv \vec v$) and our
medium is transparent for the given frequency
$\nu=\frac{\omega}{2\pi}$. Then, by
\er{MaxVacFullPPNmmmffffffiuiuhjuughbghhuiiujjjjjjjjhhhjjj} and
\er{MaxVacFullPPNmmmffffffiuiuhjuughbghhuiiujjjjjjjjhhhjjjkk} we
rewrite
\er{MaxVacFullPPNmmmffffffiuiuhjuughbghhiuijghghghhjhjhhghyuyhhjhjihikjjjjkjljkl}
as
\begin{equation}\label{MaxVacFullPPNmmmffffffiuiuhjuughbghhiuijghghghhjhjhhghyuyhhjhjihikjjj}
\begin{cases}
\frac{\partial S}{\partial t}=c\\
\frac{\partial A}{\partial t}=0.
\end{cases}
\end{equation}
Thus $\nabla_{\vec x}S$ is independent of $t$ and moreover, we
rewrite
\er{MaxVacFullPPNmmmffffffiuiuhjuughbghhiuijghghghhjhjhhghyuyjjjhhjhjff}
and
\er{MaxVacFullPPNmmmffffffiuiuhjuughbghhiuijghghghhhfhhghghguygtjjjkkjjkkhh}
as:
\begin{equation}\label{MaxVacFullPPNmmmffffffiuiuhjuughbghhiuijghghghhjhjhhghyuyiyyujj}
\frac{c^2}{c^2_0}\left(1+\frac{1}{c}\vec {\tilde u}\cdot\nabla_\vec
x S\right)^2=\left|\nabla_\vec x S\right|^2,
\end{equation}
and
\begin{equation}\label{MaxVacFullPPNmmmffffffiuiuhjuughbghhiuijghghghhhfhhghghguygtjuuujj}
2\left(\nabla_{\vec x}S-\frac{c}{c_0}\left(1+\frac{1}{c}\left(\vec
{\tilde u}\cdot\nabla_{\vec x}S\right)\right)\frac{\vec {\tilde
u}}{c_0}\right)\cdot\nabla_{\vec x}A=\left(\Div_{\vec
x}\left\{\frac{1}{c^2_0}\left(c+\vec {\tilde u}\cdot\nabla_{\vec
x}S\right)\vec {\tilde u}\right\}-\left(\Delta_{\vec
x}S\right)\right)A.
\end{equation}

In particular, in the case of the region of the space where the
following approximation is valid:
\begin{equation}\label{ojhkk}
\frac{|\vec {\tilde u}|^2}{c^2_0}\ll 1,
\end{equation}
up to order $O\left(\frac{|\vec {\tilde u}|^2}{c^2_0}\right)$, we
rewrite
\er{MaxVacFullPPNmmmffffffiuiuhjuughbghhiuijghghghhjhjhhghyuyiyyujj}
as:
\begin{equation}\label{MaxVacFullPPNmmmffffffiuiuhjuughbghhiuijghghghhjhjhhghyuyiyyujjljk}
\left|\frac{c\vec {\tilde u}}{c^2_0}-\nabla_\vec x
S\right|^2=\frac{c^2}{c^2_0},
\end{equation}
and
\er{MaxVacFullPPNmmmffffffiuiuhjuughbghhiuijghghghhhfhhghghguygtjuuujj}
as:
\begin{equation}\label{MaxVacFullPPNmmmffffffiuiuhjuughbghhiuijghghghhhfhhghghguygtjuuujjjk}
\left(\frac{c\vec {\tilde u}}{c^2_0}-\nabla_{\vec
x}S\right)\cdot\nabla_{\vec x}A+\frac{1}{2}\left(\Div_{\vec
x}\left\{\frac{c}{c^2_0}\vec {\tilde u}\right\}-\Delta_{\vec
x}S\right)A=0.
\end{equation}
The Eikonal equation
\er{MaxVacFullPPNmmmffffffiuiuhjuughbghhiuijghghghhjhjhhghyuyiyyujjljk}
and equation of the beam propagation
\er{MaxVacFullPPNmmmffffffiuiuhjuughbghhiuijghghghhhfhhghghguygtjuuujjjk}
are two basic equations of propagation of monochromatic light in the
Geometric Optics approximation inside a moving medium or/and in the
presence of non-trivial gravitational field, provided that the field
$\vec {\tilde u}$ satisfies \er{ojhkk}.

Next if we consider an arbitrary characteristic curve $\vec
r(s):[a,b]\to\mathbb{R}^3$ of equation
\er{MaxVacFullPPNmmmffffffiuiuhjuughbghhiuijghghghhhfhhghghguygtjuuujjjk}
defined as a solution of ordinary differential equation
\begin{equation}\label{MaxVacFullPPNmmmffffffiuiuhjuughbghhiuijghghghhhfhhghghguygtjuuujjjkk}
\begin{cases}
\frac{d\vec r}{ds}(s)=\frac{c}{c^2_0\left(\vec r(s)\right)}\vec
{\tilde u}\left(\vec
r(s)\right)-\nabla_{\vec x}S\left(\vec r(s)\right)\\
\vec r(a)=\vec x_0,
\end{cases}
\end{equation}
then, as before, by
\er{MaxVacFullPPNmmmffffffiuiuhjuughbghhiuijghghghhhfhhghghguygtjuuujjjk}
and
\er{MaxVacFullPPNmmmffffffiuiuhjuughbghhiuijghghghhhfhhghghguygtjuuujjjkk}
we have
\begin{equation}\label{MaxVacFullPPNmmmffffffiuiuhjuughbghhiuijghghghhhfhhghghguygtjuuujjjkhjhjh}
\frac{d}{ds}\left(A\left(\vec r(s)\right)\right)=\nabla_{\vec
x}A\left(\vec r(s)\right)\cdot\frac{d\vec
r}{ds}(s)=\frac{1}{2}\left(\Delta_{\vec x}S\left(\vec
r(s)\right)-\Div_{\vec x}\left\{\frac{c}{c^2_0\left(\vec
r(s)\right)}\vec {\tilde u}\left(\vec
r(s)\right)\right\}\right)A\left(\vec r(s)\right),
\end{equation}
that implies
\begin{equation}\label{MaxVacFullPPNmmmffffffiuiuhjuughbghhiuijghghghhhfhhghghguygtjuuujjjkhjhjhffgg}
A\left(\vec r(s)\right)=A\left(\vec
x_0\right)e^{\frac{1}{2}\int_a^{s}\left(\Delta_{\vec x}S\left(\vec
r(\tau)\right)-\Div_{\vec x}\left\{\frac{c}{c^2_0\left(\vec
r(\tau)\right)}\vec {\tilde u}\left(\vec
r(\tau)\right)\right\}\right)d\tau}\quad\quad\forall s\in[a,b].
\end{equation}
In particular,
\begin{equation}\label{MaxVacFullPPNmmmffffffiuiuhjuughbghhiuijghghghhhfhhghghguygtjuuujjjkhjhjhffggiouuiiu}
A\left(\vec x_0\right)=0\;\;\text{implies}\;\; A\left(\vec
r(s)\right)=0\quad\forall s\in[a,b],\quad\text{and}\quad A\left(\vec
x_0\right)\neq 0\;\;\text{implies}\;\; A\left(\vec r(s)\right)\neq
0\quad\forall s\in[a,b].
\end{equation}
Therefore, by
\er{MaxVacFullPPNmmmffffffiuiuhjuughbghhiuijghghghhhfhhghghguygtjuuujjjkhjhjhffggiouuiiu}
we deduce that in the case of \er{ojhkk} the curve that satisfies
\er{MaxVacFullPPNmmmffffffiuiuhjuughbghhiuijghghghhhfhhghghguygtjuuujjjkk}
coincides with the ray of light that passes through the point $\vec
x_0$. So in the case of \er{ojhkk}, equality
\er{MaxVacFullPPNmmmffffffiuiuhjuughbghhiuijghghghhhfhhghghguygtjuuujjjkk}
is the equation of a ray and the vector field $\vec h_0$, defined
for every $\vec x$ by:
\begin{equation}\label{MaxVacFullPPNmmmffffffiuiuhjuughbghhiuijghghghhhfhhghghguygtjuuujjjkyuuy}
\vec h_0(\vec x):=\frac{c}{c^2_0(\vec x)}\vec {\tilde u}(\vec
x)-\nabla_{\vec x}S(\vec x)\approx\frac{c}{c^2_0(\vec x)}\,\vec
h(\vec x),
\end{equation}
is the direction of the propagation of the ray that passes through
point $\vec x$. Moreover, by
\er{MaxVacFullPPNmmmffffffiuiuhjuughbghhiuijghghghhjhjhhghyuyiyyujjljk}
$\vec h_0$ satisfies
\begin{equation}\label{MaxVacFullPPNmmmffffffiuiuhjuughbghhiuijghghghhjhjhhghyuyiyyujjljkgghhg}
|\vec h_0|^2=\frac{c^2}{c^2_0}.
\end{equation}
\begin{remark}\label{hjgjggh}
As before in remark \ref{hjgjgghjj11nnjj}, in contrast to the proof
of
\er{MaxVacFullPPNmmmffffffiuiuhjuughbghhiuijghghghhjhjhhghyuyjjjhhjhjff},
\er{MaxVacFullPPNmmmffffffiuiuhjuughbghhiuijghghghhjhjhhghyuyiyyujj}
or
\er{MaxVacFullPPNmmmffffffiuiuhjuughbghhiuijghghghhjhjhhghyuyiyyujjljk},
we do not use the Geometric Optics approximations
\er{slochangGGaaffgfgjhjggh} or \er{slochangGGaaffgfgjhjggh11} in
the proof of
\er{MaxVacFullPPNmmmffffffiuiuhjuughbghhiuijghghghhhfhjj},
\er{MaxVacFullPPNmmmffffffiuiuhjuughbghhiuijghghghhhfhhghghguygtjjjkkjjkkhh},
\er{MaxVacFullPPNmmmffffffiuiuhjuughbghhiuijghghghhhfhhghghguygtjuuujj}
and
\er{MaxVacFullPPNmmmffffffiuiuhjuughbghhiuijghghghhhfhhghghguygtjuuujjjk}.
We just need the estimation \er{ojhkk} for the proof of
\er{MaxVacFullPPNmmmffffffiuiuhjuughbghhiuijghghghhhfhhghghguygtjuuujjjk}.
So,
\er{MaxVacFullPPNmmmffffffiuiuhjuughbghhiuijghghghhhfhhghghguygtjuuujjjkhjhjhffggiouuiiu}
is still valid without assumption of the Geometric Optics
approximations \er{slochangGGaaffgfgjhjggh} or
\er{slochangGGaaffgfgjhjggh11} and thus, the vector field $\vec
h_0$, defined by
\er{MaxVacFullPPNmmmffffffiuiuhjuughbghhiuijghghghhhfhhghghguygtjuuujjjkyuuy},
is the direction of the propagation of the ray that passes through
point $\vec x$ also in the general case, provided the estimation
\er{ojhkk} holds. However, without assumption of the Geometric
Optics approximation we cannot derive
\er{MaxVacFullPPNmmmffffffiuiuhjuughbghhiuijghghghhjhjhhghyuyiyyujjljkgghhg}
anymore.
\end{remark}

Next, under the approximation \er{ojhkk}
consider a curve $\vec
r(s):[a,b]\to\mathbb{R}^3$ with endpoints $\vec r(a)=N$ and $\vec
r(b)=M$. Then integrating the square root of both sides of
\er{MaxVacFullPPNmmmffffffiuiuhjuughbghhiuijghghghhjhjhhghyuyiyyujjljk}
over the curve $\vec r(s)$ we deduce
\begin{equation}\label{MaxMedFullGGffgggyyojjhhjkhjyyiuhggjhhjhuyytytyuuytrrtghjtyuggyuighjuyioyyfgffhyuhhghzzrrhhkkk}
\int_a^b\left|\frac{c}{c^2_0\left(\vec r(s)\right)}\vec {\tilde
u}\left(\vec r(s)\right)-\nabla_{\vec x}S\left(\vec
r(s)\right)\right|\,\left|\vec
r'(s)\right|ds=\int_a^b\frac{c}{c_0\left(\vec r(s)\right)}\left|\vec
r'(s)\right|ds.
\end{equation}
Thus in particular,
\begin{equation}\label{MaxMedFullGGffgggyyojjhhjkhjyyiuhggjhhjhuyytytyuuytrrtghjtyuggyuighjuyioyyfgffhyuhhghzzrriuihhkkk}
\int_a^b\left(\frac{c}{c^2_0\left(\vec r(s)\right)}\vec {\tilde
u}\left(\vec r(s)\right)-\nabla_{\vec x}S\left(\vec
r(s)\right)\right)\cdot\vec
r'(s)ds\leq\int_a^b\frac{c}{c_0\left(\vec r(s)\right)}\left|\vec
r'(s)\right|ds,
\end{equation}
i.e.
\begin{equation}\label{MaxMedFullGGffgggyyojjhhjkhjyyiuhggjhhjhuyytytyuuytrrtghjtyuggyuighjuyioyyfgffhyuhhghzzrrkkijjhjhhkkk}
\left(-S(M)\right)-
\left(-S(N)\right)\leq\int_a^b\frac{c}{c_0\left(\vec
r(s)\right)}\left|\vec
r'(s)\right|ds-\int_a^b\frac{c}{c^2_0\left(\vec r(s)\right)}\vec
{\tilde u}\left(\vec r(s)\right)\cdot\vec r'(s)ds.
\end{equation}
Moreover, if
\begin{equation}\label{MaxMedFullGGffgggyyojjyugggjhhjzzrrhhkkk}
\frac{d\vec r}{ds}(s)=\sigma(s)\vec h_0\left(\vec
r(s)\right):=\sigma(s)\left(\frac{c}{c^2_0\left(\vec
r(s)\right)}\vec {\tilde u}\left(\vec r(s)\right)-\nabla_{\vec
x}S\left(\vec r(s)\right)\right),
\end{equation}
for some nonnegative scalar factor $\sigma=\sigma(s)$ then by
\er{MaxMedFullGGffgggyyojjyugggjhhjzzrrhhkkk} we rewrite
\er{MaxMedFullGGffgggyyojjhhjkhjyyiuhggjhhjhuyytytyuuytrrtghjtyuggyuighjuyioyyfgffhyuhhghzzrrhhkkk}
as
\begin{equation}\label{MaxMedFullGGffgggyyojjhhjkhjyyiuhggjhhjhuyytytyuuytrrtghjtyuggyuighjuyioyyfgffhyuhhghzzrrkkhhkkk}
\left(-S(M)\right)-
\left(-S(N)\right)=\int_a^b\frac{c}{c_0\left(\vec
r(s)\right)}\left|\vec
r'(s)\right|ds-\int_a^b\frac{c}{c^2_0\left(\vec r(s)\right)}\vec
{\tilde u}\left(\vec r(s)\right)\cdot\vec r'(s)ds.
\end{equation}
Thus, by comparing
\er{MaxVacFullPPNmmmffffffiuiuhjuughbghhiuijghghghhhfhhghghguygtjuuujjjkk}
with \er{MaxMedFullGGffgggyyojjyugggjhhjzzrrhhkkk} and using
\er{MaxMedFullGGffgggyyojjhhjkhjyyiuhggjhhjhuyytytyuuytrrtghjtyuggyuighjuyioyyfgffhyuhhghzzrrkkijjhjhhkkk}
and
\er{MaxMedFullGGffgggyyojjhhjkhjyyiuhggjhhjhuyytytyuuytrrtghjtyuggyuighjuyioyyfgffhyuhhghzzrrkkhhkkk},
we deduce that if we assume that the light travel from the point $N$
to the point $M$ across the curve $\vec {\tilde
r}(s):[a,b]\to\mathbb{R}^3$ such that $\vec{\tilde r(a)}=N$ and
$\vec {\tilde r}(b)=M$, then
\begin{equation}\label{MaxMedFullGGffgggyyojjhhjkhjyyiuhggjhhjhuyytytyuuytrrtghjtyuggyuighjuyioyyfgffhyuhhghzzrrkkhhkkkhhh}
\left(-S(M)\right)-
\left(-S(N)\right)=\int_a^b\frac{c}{c_0\left(\vec{\tilde
r}(s)\right)}\left|\vec {\tilde
r}'(s)\right|ds-\int_a^b\frac{c}{c^2_0\left(\vec{\tilde
r}(s)\right)}\vec {\tilde u}\left(\vec{\tilde r}(s)\right)\cdot\vec
{\tilde r}'(s)ds,
\end{equation}
and for every other curve $\vec r(s):[a,b]\to\mathbb{R}^3$ with
endpoints $\vec r(a)=N$ and $\vec r(b)=M$ we have
\begin{multline}\label{MaxMedFullGGffgggyyojjhhjkhjyyiuhggjhhjhuyytytyuuytrrtghjtyuggyuighjuyioyyfgffhyuhhghzzrrkkhhkkkhhhjhkjhh}
\int_a^b\frac{c}{c_0\left(\vec r(s)\right)}\left|\vec
r'(s)\right|ds-\int_a^b\frac{c}{c^2_0\left(\vec r(s)\right)}\vec
{\tilde u}\left(\vec r(s)\right)\cdot\vec r'(s)ds\geq\\
\int_a^b\frac{c}{c_0\left(\vec{\tilde r}(s)\right)}\left|\vec
{\tilde r}'(s)\right|ds-\int_a^b\frac{c}{c^2_0\left(\vec{\tilde
r}(s)\right)}\vec {\tilde u}\left(\vec{\tilde r}(s)\right)\cdot\vec
{\tilde r}'(s)ds.
\end{multline}
So we obtain the following Fermat Principle:
\begin{proposition}\label{gughghf}
Assume Geometric Optics approximation together with \er{ojhkk}. Then
the light that travels from point $N$ to point $M$ chooses the path
$\vec r(s):[a,b]\to\mathbb{R}^3$ with endpoints $\vec r(a)=N$ and
$\vec r(b)=M$ which minimizes the quantity:
\begin{equation}\label{MaxMedFullGGffgggyyojjhhjkhjyyiuhggjhhjhuyytytyuuytrrtghjtyuggyuighjuyioyyfgffhyuhhghzzrrkkhhkkkhhhjhkjhhghhggh}
J\left(\vec r(\cdot)\right):=\int_a^bn\left(\vec
r(s)\right)\left|\vec r'(s)\right|ds-\int_a^b
\frac{1}{c}n^2\left(\vec r(s)\right)\vec {\tilde u}\left(\vec
r(s)\right)\cdot\vec r'(s)ds,
\end{equation}
where we set the refraction index:
\begin{equation}\label{MaxMedFullGGffgggyyojjhhjkhjyyiuhggjhhjhuyytytyuuytrrtghjtyuggyuighjuyioyyfgffhyuhhghzzrrkkhhkkkhhhjhkjhhghhgghiuiu1}
n\left(\vec x\right):=\frac{c}{c_0\left(\vec x\right)}.
\end{equation}
Moreover, if $\vec r(s):[a,b]\to\mathbb{R}^3$ with endpoints $\vec
r(a)=N$ and $\vec r(b)=M$ is the real path of the light, then:
\begin{equation}\label{MaxMedFullGGffgggyyojjhhjkhjyyiuhggjhhjhuyytytyuuytrrtghjtyuggyuighjuyioyyfgffhyuhhghzzrrkkhhkkkoiioio}
\left(-S(M)\right)- \left(-S(N)\right)=\int_a^bn\left(\vec
r(s)\right)\left|\vec r'(s)\right|ds-\int_a^b
\frac{1}{c}n^2\left(\vec r(s)\right)\vec {\tilde u}\left(\vec
r(s)\right)\cdot\vec r'(s)ds.
\end{equation}
\end{proposition}
In particular, by Proposition \ref{gughghf} the path of travel of
the light satisfies the Euler-Lagrange equation for the functional
$J\left(\vec r(\cdot)\right)$:
\begin{multline}\label{MaxMedFullGGffgggyyojjhhjkhjyyiuhggjhhjhuyytytyuuytrrtghjtyuggyuighjuyioyyfgffhyuhhghzzrrkkhhkkkhhhjhkjhhghhgghiuiu2}
\frac{d}{ds}\left(n\left(\vec r(s)\right)\frac{1}{\left|\vec
r'(s)\right|}\vec r'(s)-\frac{1}{c}n^2\left(\vec r(s)\right)\vec
{\tilde u}\left(\vec r(s)\right)\right)=\\
\left|\vec r'(s)\right|\nabla_{\vec x}n\left(\vec
r(s)\right)-\frac{2}{c}\left(\vec {\tilde u}\left(\vec
r(s)\right)\cdot\vec r'(s)\right)n\left(\vec r(s)\right)\nabla_{\vec
x}n\left(\vec r(s)\right)-\frac{1}{c}n^2\left(\vec
r(s)\right)\left\{d_{\vec x}\vec {\tilde u}\left(\vec
r(s)\right)\right\}^T\cdot\vec r'(s),
\end{multline}
that we rewrite as:
\begin{multline}\label{MaxMedFullGGffgggyyojjhhjkhjyyiuhggjhhjhuyytytyuuytrrtghjtyuggyuighjuyioyyfgffhyuhhghzzrrkkhhkkkhhhjhkjhhghhgghiuiuhhjhj1}
\frac{1}{\left|\vec r'(s)\right|}\frac{d}{ds}\left(n\left(\vec
r(s)\right)\frac{1}{\left|\vec r'(s)\right|}\vec r'(s)\right)=\\
\nabla_{\vec x}n\left(\vec r(s)\right)+\frac{1}{c}n^2\left(\vec
r(s)\right)\left(d_{\vec x}\vec {\tilde u}\left(\vec
r(s)\right)-\left\{d_{\vec x}\vec {\tilde u}\left(\vec
r(s)\right)\right\}^T\right)\cdot\left(\frac{1}{\left|\vec
r'(s)\right|}\vec r'(s)\right)\\
+\frac{2}{c}n\left(\vec r(s)\right)\left\{\vec {\tilde u}\left(\vec
r(s)\right)\otimes\nabla_{\vec x}n\left(\vec
r(s)\right)-\nabla_{\vec x}n\left(\vec r(s)\right)\otimes\vec
{\tilde u}\left(\vec
r(s)\right)\right\}\cdot\left(\frac{1}{\left|\vec r'(s)\right|}\vec
r'(s)\right).
\end{multline}
Therefor by \er{apfrm9} and
\er{MaxMedFullGGffgggyyojjhhjkhjyyiuhggjhhjhuyytytyuuytrrtghjtyuggyuighjuyioyyfgffhyuhhghzzrrkkhhkkkhhhjhkjhhghhgghiuiuhhjhj1}
we deduce the differential equation of the path of light:
\begin{multline}\label{MaxMedFullGGffgggyyojjhhjkhjyyiuhggjhhjhuyytytyuuytrrtghjtyuggyuighjuyioyyfgffhyuhhghzzrrkkhhkkkhhhjhkjhhghhgghiuiuhhjhj}
\frac{d}{d\lambda}\left(n\left(\vec r\right)\frac{d\vec
r}{d\lambda}\right)=\frac{1}{c}n^2\left(\vec
r\right)\left(curl_{\vec x}\vec {\tilde u}\left(\vec
r\right)\right)\times\frac{d\vec r}{d\lambda}\\ +\nabla_{\vec
x}n\left(\vec r\right) +\frac{2}{c}n\left(\vec r\right)\left\{\vec
{\tilde u}\left(\vec r\right)\otimes\nabla_{\vec x}n\left(\vec
r\right)-\nabla_{\vec x}n\left(\vec r\right)\otimes\vec {\tilde
u}\left(\vec r\right)\right\}\cdot\frac{d\vec r}{d\lambda},
\end{multline}
where
\begin{equation}\label{MaxMedFullGGffgggyyojjhhjkhjyyiuhggjhhjhuyytytyuuytrrtghjtyuggyuighjuyioyyfgffhyuhhghzzrrkkhhkkkhhhjhkjhhghhgghiuiu}
\lambda:=\int_a^s\left|\vec r'(\tau)\right|d\tau,
\end{equation}
is the natural parameter of the curve.

Next, assume that the wave we consider has an electromagnetic
nature. Then by \er{gughhghfbvnbv} and \er{uyuyuyy} we have
\begin{equation}\label{gughhghfbvnbvyyuuyr}
c_0=c\sqrt{\kappa_0\gamma_0}\quad\text{and}\quad\vec {\tilde
u}=\left(\gamma_0\kappa_0\vec v+(1-\gamma_0\kappa_0)\vec u\right),
\end{equation}
where, $\vec u$ is the medium velocity and $\vec v$ is the local
vectorial gravitational potential. Moreover, assume that we consider
light traveling in some region either filled with the resting medium
of constant dielectric permeability $\gamma_0$ and magnetic
permeability $\kappa_0$ or in the vacuum. Then by
\er{gughhghfbvnbvyyuuyr} and
\er{MaxMedFullGGffgggyyojjhhjkhjyyiuhggjhhjhuyytytyuuytrrtghjtyuggyuighjuyioyyfgffhyuhhghzzrrkkhhkkkhhhjhkjhhghhgghiuiu1}
we have:
\begin{equation}\label{gughhghfbvnbvyyuuyrhiyyuiuu}
n=\frac{1}{\sqrt{\kappa_0\gamma_0}}\;\;\;\text{is a
constant,}\quad\text{and}\quad\vec {\tilde u}=\gamma_0\kappa_0\vec
v,
\end{equation}
Then by \er{gughhghfbvnbvyyuuyrhiyyuiuu} we rewrite
\er{MaxMedFullGGffgggyyojjhhjkhjyyiuhggjhhjhuyytytyuuytrrtghjtyuggyuighjuyioyyfgffhyuhhghzzrrkkhhkkkhhhjhkjhhghhgghiuiuhhjhj}
as:
\begin{equation}\label{MaxMedFullGGffgggyyojjhhjkhjyyiuhggjhhjhuyytytyuuytrrtghjtyuggyuighjuyioyyfgffhyuhhghzzrrkkhhkkkhhhjhkjhhghhgghiuiuhhjhjiuiuyu}
\frac{d^2\vec
r}{d\lambda^2}=\frac{1}{c}\sqrt{\gamma_0\kappa_0}\left(curl_{\vec
x}\vec v\left(\vec r\right)\right)\times\frac{d\vec
r}{d\lambda}=\frac{1}{c}\,\frac{1}{n}\,\left(curl_{\vec x}\vec
v\left(\vec r\right)\right)\times\frac{d\vec r}{d\lambda}.
\end{equation}
In particular, if our coordinate system is inertial, or more
generally non-rotating, then $curl_{\vec x}\vec v=0$ and we deduce
that the path of the light from the point $N$ to the point $M$ is
the direct line connecting these points, provided we take in the
account estimation \er{ojhkk}.

On the other hand, if our system is rotating, then, since $\vec v$
is a speed-like vector field, we clearly deduce:
\begin{equation}\label{MaxMedFullGGffgggyyojjhhjkhjyyiuhggjhhjhuyytytyuuytrrtghjtyuggyuighjuyioyyfgffhyuhhghzzrrkkhhkkkhhhjhkjhhghhgghiuiuhhjhjiuiuyujjk}
curl_{\vec x}\vec v=-2\vec w,
\end{equation}
where $\vec w$ is the vector of the angular speed of rotation of our
coordinate system. Thus by inserting
\er{MaxMedFullGGffgggyyojjhhjkhjyyiuhggjhhjhuyytytyuuytrrtghjtyuggyuighjuyioyyfgffhyuhhghzzrrkkhhkkkhhhjhkjhhghhgghiuiuhhjhjiuiuyujjk}
into
\er{MaxMedFullGGffgggyyojjhhjkhjyyiuhggjhhjhuyytytyuuytrrtghjtyuggyuighjuyioyyfgffhyuhhghzzrrkkhhkkkhhhjhkjhhghhgghiuiuhhjhjiuiuyu}
we deduce:
\begin{equation}\label{MaxMedFullGGffgggyyojjhhjkhjyyiuhggjhhjhuyytytyuuytrrtghjtyuggyuighjuyioyyfgffhyuhhghzzrrkkhhkkkhhhjhkjhhghhgghiuiuhhjhjiuiuyuyyu}
\frac{d^2\vec
r}{d\lambda^2}=-\frac{2}{c}\sqrt{\gamma_0\kappa_0}\,\vec
w\times\frac{d\vec r}{d\lambda}.
\end{equation}
In particular, by
\er{MaxMedFullGGffgggyyojjhhjkhjyyiuhggjhhjhuyytytyuuytrrtghjtyuggyuighjuyioyyfgffhyuhhghzzrrkkhhkkkhhhjhkjhhghhgghiuiuhhjhjiuiuyuyyu}
if we consider that $\vec w=(0,0,w)$ and $\vec r=(x,y,z)$, then
there exist three dimensionless constants $C_1$, $C_2$ and $C_3$
such that
\begin{equation}\label{MaxMedFullGGffgggyyojjhhjkhjyyiuhggjhhjhuyytytyuuytrrtghjtyuggyuighjuyioyyfgffhyuhhghzzrrkkhhkkkhhhjhkjhhghhgghiuiuhhjhjiuiuyuyyukkl}
\begin{cases}
\frac{dx}{d\lambda}=-C_1\sin{\left(\frac{2w}{c}\sqrt{\gamma_0\kappa_0}\,\lambda\right)}+C_2\cos{\left(\frac{2w}{c}\sqrt{\gamma_0\kappa_0}\,\lambda\right)}
\\
\frac{dy}{d\lambda}=-C_1\cos{\left(\frac{2w}{c}\sqrt{\gamma_0\kappa_0}\,\lambda\right)}-C_2\sin{\left(\frac{2w}{c}\sqrt{\gamma_0\kappa_0}\,\lambda\right)}
\\
\frac{dz}{d\lambda}=C_3,
\end{cases}
\end{equation}
and moreover, since $\lambda$ is a natural parameter, the constants
satisfy:
\begin{equation}\label{MaxMedFullGGffgggyyojjhhjkhjyyiuhggjhhjhuyytytyuuytrrtghjtyuggyuighjuyioyyfgffhyuhhghzzrrkkhhkkkhhhjhkjhhghhgghiuiuhhjhjiuiuyuyyuojk}
C^2_1+C^2_2+C^2_3=1.
\end{equation}
Then by
\er{MaxMedFullGGffgggyyojjhhjkhjyyiuhggjhhjhuyytytyuuytrrtghjtyuggyuighjuyioyyfgffhyuhhghzzrrkkhhkkkhhhjhkjhhghhgghiuiuhhjhjiuiuyuyyukkl}
there exist three additional constants $D_1$, $D_2$ and $D_3$ such
that
\begin{equation}\label{MaxMedFullGGffgggyyojjhhjkhjyyiuhggjhhjhuyytytyuuytrrtghjtyuggyuighjuyioyyfgffhyuhhghzzrrkkhhkkkhhhjhkjhhghhgghiuiuhhjhjiuiuyuyyukklghhg}
\begin{cases}
x(\lambda)=C_1\frac{c}{2w}\sqrt{\gamma_0\kappa_0}\,\left(\cos{\left(\frac{2w}{c}\sqrt{\gamma_0\kappa_0}\,\lambda\right)}-1\right)+C_2\frac{c}{2w}\sqrt{\gamma_0\kappa_0}\,\sin{\left(\frac{2w}{c}\sqrt{\gamma_0\kappa_0}\,\lambda\right)}+D_1
\\
y(\lambda)=-C_1\frac{c}{2w}\sqrt{\gamma_0\kappa_0}\,\sin{\left(\frac{2w}{c}\sqrt{\gamma_0\kappa_0}\,\lambda\right)}+C_2\frac{c}{2w}\sqrt{\gamma_0\kappa_0}\,\left(\cos{\left(\frac{2w}{c}\sqrt{\gamma_0\kappa_0}\,\lambda\right)}-1\right)+D_2
\\
z(\lambda)=C_3\lambda+D_3.
\end{cases}
\end{equation}
So, the curve in
\er{MaxMedFullGGffgggyyojjhhjkhjyyiuhggjhhjhuyytytyuuytrrtghjtyuggyuighjuyioyyfgffhyuhhghzzrrkkhhkkkhhhjhkjhhghhgghiuiuhhjhjiuiuyuyyukklghhg}
is the trajectory of the light in the rotating coordinate system,
provided we assume \er{ojhkk}. In particular, by
\er{MaxMedFullGGffgggyyojjhhjkhjyyiuhggjhhjhuyytytyuuytrrtghjtyuggyuighjuyioyyfgffhyuhhghzzrrkkhhkkkhhhjhkjhhghhgghiuiuhhjhjiuiuyuyyukklghhg}
and
\er{MaxMedFullGGffgggyyojjhhjkhjyyiuhggjhhjhuyytytyuuytrrtghjtyuggyuighjuyioyyfgffhyuhhghzzrrkkhhkkkhhhjhkjhhghhgghiuiuhhjhjiuiuyuyyukkl}
we have:
\begin{equation}\label{MaxMedFullGGffgggyyojjhhjkhjyyiuhggjhhjhuyytytyuuytrrtghjtyuggyuighjuyioyyfgffhyuhhghzzrrkkhhkkkhhhjhkjhhghhgghiuiuhhjhjiuiuyuyyuojkjkkkhgugg}
\begin{cases}
x(0)=D_1,\quad y(0)=D_2,\quad z(0)=D_3,\\
\frac{dx}{d\lambda}(0)=C_2,\quad\frac{dy}{d\lambda}(0)=-C_1,\quad\frac{dz}{d\lambda}(0)=C_3.
\end{cases}
\end{equation}
%
%
%
%
%
%
The constants $C_1,C_2,C_3,D_1,D_2,D_3$ can be determined either by
the initial data
\er{MaxMedFullGGffgggyyojjhhjkhjyyiuhggjhhjhuyytytyuuytrrtghjtyuggyuighjuyioyyfgffhyuhhghzzrrkkhhkkkhhhjhkjhhghhgghiuiuhhjhjiuiuyuyyuojkjkkkhgugg}
or by the beginning and the ending points $N$ and $M$ of the curve.

\subsubsection{The laws of reflection and refraction}\label{rrGO}
Next consider a
monochromatic wave of the frequency $\nu=\omega/(2\pi)$
characterized by:
\begin{equation}\label{MaxVacFullPPNmmmffffffiuiuhjuughbghhuiiujjhhjjhjhhj}
U(\vec x,t)=A(\vec x)e^{ik_0S(\vec x,t)},\quad\text{where}\quad
k_0=\frac{\omega}{c}\quad\text{and}\quad\frac{\partial S}{\partial
t}=c\,,
\end{equation}
and, consistently with
\er{MaxVacFullPPNmmmffffffiuiuhjuughbghhiuijghghghhhfhhghghguygtjuuujjjkyuuy}
consider a direction field:
\begin{equation}\label{MaxVacFullPPNmmmffffffiuiuhjuughbghhiuijghghghhhfhhghghguygtjuuujjjkyuuykjkjj}
\vec h_0(\vec x)=\frac{c}{c^2_0(\vec x)}\vec {\tilde u}(\vec
x)-\nabla_{\vec x}S(\vec x)\approx\frac{c}{c^2_0(\vec x)}\vec h(\vec
x).
\end{equation}
Furthermore, assume that this wave undergoes reflection and/or
refraction on the stationary (time independent) surface
$\mathcal{T}$ with the outcoming unit normal $\vec n$, separating
two regions characterized respectively by $c_0=c^{(1)}_0$ and $\vec
{\tilde u}=\vec {\tilde u}_1$ and by $c^{(2)}_0$ and $\vec {\tilde
u}_2$, with the formation of the reflected wave (of the same
frequency), characterized by:
\begin{equation}\label{MaxVacFullPPNmmmffffffiuiuhjuughbghhuiiujjhhjjhjhhjugh}
U_1(\vec x,t)=A_1(\vec x)e^{ik_0S_1(\vec
x,t)},\quad\text{where}\quad\frac{\partial S_1}{\partial t}=c\,,
\end{equation}
and by a direction field:
\begin{equation}\label{MaxVacFullPPNmmmffffffiuiuhjuughbghhiuijghghghhhfhhghghguygtjuuujjjkyuuykjkjjhjhhj}
\vec h_1(\vec x)=\frac{c}{c^2_0(\vec x)}\vec {\tilde u}(\vec
x)-\nabla_{\vec x}S_1(\vec x),
\end{equation}
and formation of the refracted wave (of the same frequency),
characterized by:
\begin{equation}\label{MaxVacFullPPNmmmffffffiuiuhjuughbghhuiiujjhhjjhjhhj1}
U_2(\vec x,t)=A_2(\vec x)e^{ik_0S_2(\vec
x,t)},\quad\text{where}\quad\frac{\partial S_2}{\partial t}=c\,.
\end{equation}
and by a direction field:
\begin{equation}\label{MaxVacFullPPNmmmffffffiuiuhjuughbghhiuijghghghhhfhhghghguygtjuuujjjkyuuykjkjj1}
\vec h_2(\vec x)=\frac{c}{\left(c^{(2)}_0(\vec x)\right)^2}\vec
{\tilde u_2}(\vec x)-\nabla_{\vec x}S_2(\vec x).
\end{equation}
Then the boundary conditions of $U$, $U_1$ and $U_2$ depend on the
physical meaning of these fields. However, one of the
\underline{necessary} conditions should be that
\begin{equation}\label{MaxMedFullGGffgggyyojjyugggjhhjiiuuiyuyuyyuyuzzhhkkk}
S_1(\vec x,t)=S_2(\vec x,t)+C_2=S(\vec x,t)\quad\quad\forall\vec
x\in\mathcal{T},
\end{equation}
where $C_2$ is a real constant. In particular
\er{MaxMedFullGGffgggyyojjyugggjhhjiiuuiyuyuyyuyuzzhhkkk} implies
\begin{equation}\label{MaxMedFullGGffgggyyojjyugggjhhjiiuuiyuyuyyuyuzzhhkkkgdfg}
\nabla_{\vec x} S_1-\left(\vec n\cdot \nabla_{\vec x} S_1\right)\vec
n
=\nabla_{\vec x} S_2-\left(\vec n\cdot \nabla_{\vec x}
S_2\right)\vec n=\nabla_{\vec x} S-\left(\vec n\cdot \nabla_{\vec x}
S\right)\vec n\quad\quad\forall\vec x\in\mathcal{T}.
\end{equation}
In particular, for every point on the surface $\mathcal{T}$ vectors
$\nabla_{\vec x} S_1$ and $\nabla_{\vec x} S_2$ lie in the plane
formed by vectors $\vec n$ and $\nabla_{\vec x} S$. Moreover, by
\er{MaxVacFullPPNmmmffffffiuiuhjuughbghhiuijghghghhhfhhghghguygtjuuujjjkyuuykjkjj},
\er{MaxVacFullPPNmmmffffffiuiuhjuughbghhiuijghghghhhfhhghghguygtjuuujjjkyuuykjkjjhjhhj}
and \er{MaxMedFullGGffgggyyojjyugggjhhjiiuuiyuyuyyuyuzzhhkkkgdfg} we
have
\begin{equation}\label{MaxMedFullGGffgggyyojjyugggjhhjiiuuiyuyuyyuyuzzhhkkkgdfghgghgh}
\vec h_1-\left(\vec n\cdot \vec h_1\right)\vec n=\vec h_0-\left(\vec
n\cdot \vec h_0\right)\vec n\quad\quad\forall\vec x\in\mathcal{T},
\end{equation}
and in particular, for every point on the surface $\mathcal{T}$
vector $\vec h_1$ lies in the plane formed by vectors $\vec n$ and
$\vec h_0$.
Next, assume that the approximate equations in
\er{MaxVacFullPPNmmmffffffiuiuhjuughbghhiuijghghghhjhjhhghyuyiyyujjljk}
and
\er{MaxVacFullPPNmmmffffffiuiuhjuughbghhiuijghghghhhfhhghghguygtjuuujjjk}
are valid in every of two regions on the both sides of
$\mathcal{T}$. Then by
\er{MaxVacFullPPNmmmffffffiuiuhjuughbghhiuijghghghhjhjhhghyuyiyyujjljkgghhg}
we have
\begin{equation}\label{MaxVacFullPPNmmmffffffiuiuhjuughbghhiuijghghghhjhjhhghyuyiyyujjljkgghhguyuy}
|\vec h_1|=|\vec h_0|=\frac{c}{c_0}.
\end{equation}
Then, since $\vec h_1\neq\vec h_0$, by
\er{MaxMedFullGGffgggyyojjyugggjhhjiiuuiyuyuyyuyuzzhhkkkgdfghgghgh}
and
\er{MaxVacFullPPNmmmffffffiuiuhjuughbghhiuijghghghhjhjhhghyuyiyyujjljkgghhguyuy}
we deduce
\begin{equation}\label{MaxMedFullGGffgggyyojjyugggjhhjiiuuiyuyuyyuyuzzhhkkkgdfghgghghijhjhhkh}
\vec n\cdot \vec h_1=-\vec n\cdot \vec h_0\quad\quad\forall\vec
x\in\mathcal{T}.
\end{equation}
So, by
\er{MaxVacFullPPNmmmffffffiuiuhjuughbghhiuijghghghhjhjhhghyuyiyyujjljkgghhguyuy}
and
\er{MaxMedFullGGffgggyyojjyugggjhhjiiuuiyuyuyyuyuzzhhkkkgdfghgghghijhjhhkh}
we obtain the law of reflection: vector $\vec h_1$ lies in the plane
formed by vectors $\vec n$ and $\vec h_0$, and we have:
\begin{equation}\label{MaxMedFulljhhjjjjj}
\theta\left(\vec h_0,-\vec n\right)=\theta_1\left(\vec h_1,\vec
n\right)
\end{equation}
where $\theta\left(\vec h_0,-\vec n\right)$ is the angle between the
incoming ray direction $\vec h_0$ and the incoming normal to the
surface $-\vec n$ and $\theta_1\left(\vec h_1,\vec n\right)$ is the
angle between the reflected ray direction $\vec h_1$ and the
outcoming normal $\vec n$.

Next assume that the wave we consider in
\er{MaxVacFullPPNmmmffffffiuiuhjuughbghhuiiujjhhjjhjhhj} has an
electromagnetic nature. Then by \er{gughhghfbvnbvyyuuyr} we have
\begin{equation}\label{gughhghfbvnbvyyuuy}
c_0=c\sqrt{\kappa_0\gamma_0}\quad\text{and}\quad\vec {\tilde
u}=\left(\gamma_0\kappa_0\vec v+(1-\gamma_0\kappa_0)\vec u\right),
\end{equation}
where, $\vec u$ is the medium velocity and $\vec v$ is the local
vectorial gravitational potential. Similarly, on the second side of
surface $\mathcal{T}$ we have
\begin{equation}\label{gughhghfbvnbvyyuuyGGHHG}
c^{(2)}_0=c\sqrt{\kappa^{(2)}_0\gamma^{(2)}_0}\quad\text{and}\quad\vec
{\tilde u}_2=\left(\gamma^{(2)}_0\kappa^{(2)}_0\vec
v+(1-\gamma^{(2)}_0)\kappa^{(2)}_0\vec u^{(2)}\right),
\end{equation}
where, $\vec u^{(2)}$ is the medium velocity on the second side of
surface $\mathcal{T}$. Furthermore, assume that the medium rests on
the both sides of surface $\mathcal{T}$. Then in this particular
case we rewrite \er{gughhghfbvnbvyyuuy} and
\er{gughhghfbvnbvyyuuyGGHHG} as
\begin{equation}\label{gughhghfbvnbvyyuuyggg}
c_0=c\sqrt{\kappa_0\gamma_0}\quad\text{and}\quad\vec {\tilde
u}=\gamma_0\kappa_0\vec v,
\end{equation}
and
\begin{equation}\label{gughhghfbvnbvyyuuyGGHHGhjhjh}
c^{(2)}_0=c\sqrt{\kappa^{(2)}_0\gamma^{(2)}_0}\quad\text{and}\quad\vec
{\tilde u}_2=\gamma^{(2)}_0\kappa^{(2)}_0\vec v,
\end{equation}
Then in particular, by \er{gughhghfbvnbvyyuuyggg} and
\er{gughhghfbvnbvyyuuyGGHHGhjhjh} we deduce
\begin{equation}\label{gughhghfbvnbvyyuuyGGHHGhjhjhhhjhj}
\frac{c}{\left(c^{(2)}_0\right)^2}\vec {\tilde
u}_2=\frac{c}{c^2_0}\vec {\tilde u}=\frac{1}{c}\,\vec v.
\end{equation}
Thus, by inserting
\er{MaxVacFullPPNmmmffffffiuiuhjuughbghhiuijghghghhhfhhghghguygtjuuujjjkyuuykjkjj}
and \er{gughhghfbvnbvyyuuyGGHHGhjhjhhhjhj} into
\er{MaxMedFullGGffgggyyojjyugggjhhjiiuuiyuyuyyuyuzzhhkkkgdfg}, we
deduce:
\begin{equation}\label{MaxMedFullGGffgggyyojjyugggjhhjiiuuiyuyuyyuyuzzhhkkkgdfg1}
\vec h_2-\left(\vec n\cdot \vec h_2\right)\vec n=\vec h_0-\left(\vec
n\cdot \vec h_0\right)\vec n\quad\quad\forall\vec x\in\mathcal{T},
\end{equation}
and in particular, for every point on the surface $\mathcal{T}$
vector $\vec h_2$ lies in the plane formed by vectors $\vec n$ and
$\vec h_0$. On the other hand by
\er{MaxVacFullPPNmmmffffffiuiuhjuughbghhiuijghghghhjhjhhghyuyiyyujjljkgghhg}
we have:
\begin{equation}\label{MaxVacFullPPNmmmffffffiuiuhjuughbghhiuijghghghhjhjhhghyuyiyyujjljkgghhglkjlkjkjljljhhjj}
|\vec h_0|=\frac{c}{c_0}\quad\text{and}\quad|\vec
h_2|=\frac{c}{c^{(2)}_0}.
\end{equation}
So, by
\er{MaxMedFullGGffgggyyojjyugggjhhjiiuuiyuyuyyuyuzzhhkkkgdfg1} and
\er{MaxVacFullPPNmmmffffffiuiuhjuughbghhiuijghghghhjhjhhghyuyiyyujjljkgghhglkjlkjkjljljhhjj},
 we have the
Snell's law of refraction: vector $\vec h_2$ lies in the plane
formed by vectors $\vec n$ and $\vec h_0$, and we have:
\begin{equation}\label{MaxMedFulljhhjjjjjhhhj}
n\sin{\left(\theta\left(\vec h_0,\vec
n\right)\right)}=n_2\sin{\left(\theta_2\left(\vec h_2,\vec
n\right)\right)}
\end{equation}
where $\theta\left(\vec h_0,\vec n\right)$ is the angle between the
incoming ray direction $\vec h_0$ and the normal to the surface
$\vec n$, $\theta_2\left(\vec h_2,\vec n\right)$ is the angle
between the refracted ray direction $\vec h_2$ and the normal $\vec
n$ and as in
\er{MaxMedFullGGffgggyyojjhhjkhjyyiuhggjhhjhuyytytyuuytrrtghjtyuggyuighjuyioyyfgffhyuhhghzzrrkkhhkkkhhhjhkjhhghhgghiuiu1}
we set refraction indexes:
\begin{equation}\label{MaxMedFullGGffgggyyojjhhjkhjyyiuhggjhhjhuyytytyuuytrrtghjtyuggyuighjuyioyyfgffhyuhhghzzrrkkhhkkkhhhjhkjhhghhgghiuiujjkjk}
n:=\frac{c}{c_0}\quad\text{and}\quad n_2:=\frac{c}{c^{(2)}_0}.
\end{equation}

\subsubsection{Sagnac effect}\label{seGO}
Assume again the monochromatic electromagnetic wave of the frequency
$\nu=\omega/(2\pi)$ characterized by:
\begin{equation}\label{MaxVacFullPPNmmmffffffiuiuhjuughbghhuiiujjhhjjhjhhjhjjhhjhjjh}
U(\vec x,t)=A(\vec x,t)e^{iT(\vec x,t)}=A(\vec x,t)e^{ik_0S(\vec
x,t)},\quad\text{where}\quad
k_0=\frac{\omega}{c}\quad\text{and}\quad\frac{\partial S}{\partial
t}=c\,.
\end{equation}
Then by \er{gughhghfbvnbvyyuuyr} we have
\begin{equation}\label{gughhghfbvnbvyyuuyrhjhjhj}
c_0=c\sqrt{\kappa_0\gamma_0}\quad\text{and}\quad\vec {\tilde
u}=\left(\gamma_0\kappa_0\vec v+(1-\gamma_0\kappa_0)\vec u\right),
\end{equation}
where, $\vec u$ is the medium velocity and $\vec v$ is the local
vectorial gravitational potential. Moreover, assume again that we
consider light traveling in some region either filled with the
resting medium of constant dielectric permeability $\gamma_0$ and
magnetic permeability $\kappa_0$ or in the vacuum. Then by
\er{gughhghfbvnbvyyuuyrhjhjhj} and
\er{MaxMedFullGGffgggyyojjhhjkhjyyiuhggjhhjhuyytytyuuytrrtghjtyuggyuighjuyioyyfgffhyuhhghzzrrkkhhkkkhhhjhkjhhghhgghiuiu1}
we have
\begin{equation}\label{gughhghfbvnbvyyuuyrhiyyuiuuhjh}
n=\frac{1}{\sqrt{\kappa_0\gamma_0}}\;\;\;\text{is a
constant,}\quad\text{and}\quad\vec {\tilde u}=\gamma_0\kappa_0\vec
v.
\end{equation}
Next, assume that the light travels from point $N$ to point $M$
across the curve $\vec r(s):[a,b]\to\mathbb{R}^3$ with endpoints
$\vec r(a)=N$ and $\vec r(b)=M$ undergoing possibly certain number
of reflections from mirrors during its travel. Then by
\er{MaxMedFullGGffgggyyojjhhjkhjyyiuhggjhhjhuyytytyuuytrrtghjtyuggyuighjuyioyyfgffhyuhhghzzrrkkhhkkk},
\er{gughhghfbvnbvyyuuyrhiyyuiuuhjh} and
\er{MaxMedFullGGffgggyyojjyugggjhhjiiuuiyuyuyyuyuzzhhkkk} we have:
\begin{equation}\label{MaxMedFullGGffgggyyojjhhjkhjyyiuhggjhhjhuyytytyuuytrrtghjtyuggyuighjuyioyyfgffhyuhhghzzrrkkhhkkkhjjh}
\delta(-S):=\left(-S(M^-)\right)-
\left(-S(N^+)\right)=\frac{1}{\sqrt{\kappa_0\gamma_0}}\int_a^b
\left|\vec r'(s)\right|ds-\frac{1}{c}\int_a^b\vec v\left(\vec
r(s)\right)\cdot\vec r'(s)ds.
\end{equation}
In particular, if we assume that $M=N$ i.e. our curve is closed and
moreover, our curve is the boundary of some surface $\mathcal{S}_0$,
then by Stokes Theorem we have:
\begin{multline}\label{MaxMedFullGGffgggyyojjhhjkhjyyiuhggjhhjhuyytytyuuytrrtghjtyuggyuighjuyioyyfgffhyuhhghzzrrkkhhkkkhjjhjj}
\delta(-S)=\left(-S(M^-)\right)-
\left(-S(M^+)\right)=\frac{1}{\sqrt{\kappa_0\gamma_0}}\int_a^b
\left|\vec r'(s)\right|ds-\frac{1}{c}\iint \left(curl_{\vec x}\vec
v\right)\cdot\vec n
\,d\mathcal{S}_0\\=\frac{1}{\sqrt{\kappa_0\gamma_0}}\left|\partial
\mathcal{S}_0\right|-\frac{1}{c}\iint \left(curl_{\vec x}\vec
v\right)\cdot\vec n \,d\mathcal{S}_0,
\end{multline}
where $\vec n$ is the unit normal to the surface.
In particular, if our
coordinate system is inertial, or more generally non-rotating, then
$curl_{\vec x}\vec v=0$ and by
\er{MaxMedFullGGffgggyyojjhhjkhjyyiuhggjhhjhuyytytyuuytrrtghjtyuggyuighjuyioyyfgffhyuhhghzzrrkkhhkkkhjjhjj}
we deduce
\begin{equation}\label{MaxMedFullGGffgggyyojjhhjkhjyyiuhggjhhjhuyytytyuuytrrtghjtyuggyuighjuyioyyfgffhyuhhghzzrrkkhhkkkhjjhjjhjhhjjhjjh}
\delta(-S)=\frac{1}{\sqrt{\kappa_0\gamma_0}}\left|\partial
\mathcal{S}_0\right|.
\end{equation}
On the other hand, if our system is rotating, then as in
\er{MaxMedFullGGffgggyyojjhhjkhjyyiuhggjhhjhuyytytyuuytrrtghjtyuggyuighjuyioyyfgffhyuhhghzzrrkkhhkkkhhhjhkjhhghhgghiuiuhhjhjiuiuyujjk}
we clearly deduce:
\begin{equation}\label{MaxMedFullGGffgggyyojjhhjkhjyyiuhggjhhjhuyytytyuuytrrtghjtyuggyuighjuyioyyfgffhyuhhghzzrrkkhhkkkhhhjhkjhhghhgghiuiuhhjhjiuiuyujjkjjk}
curl_{\vec x}\vec v=-2\vec w,
\end{equation}
where $\vec w$ is the vector of the angular speed of rotation of our
coordinate system. Then by
\er{MaxMedFullGGffgggyyojjhhjkhjyyiuhggjhhjhuyytytyuuytrrtghjtyuggyuighjuyioyyfgffhyuhhghzzrrkkhhkkkhhhjhkjhhghhgghiuiuhhjhjiuiuyujjkjjk}
and
\er{MaxMedFullGGffgggyyojjhhjkhjyyiuhggjhhjhuyytytyuuytrrtghjtyuggyuighjuyioyyfgffhyuhhghzzrrkkhhkkkhjjhjj}
we deduce
\begin{equation}\label{MaxMedFullGGffgggyyojjhhjkhjyyiuhggjhhjhuyytytyuuytrrtghjtyuggyuighjuyioyyfgffhyuhhghzzrrkkhhkkkhjjhjjjjk}
\delta(-S)=\frac{1}{\sqrt{\kappa_0\gamma_0}}\left|\partial
\mathcal{S}_0\right|+\frac{2}{c}\iint \vec w\cdot\vec n
\,d\mathcal{S}_0.
\end{equation}
In particular, if the surface $\mathcal{S}_0$ is a part of some
plain then we rewrite
\er{MaxMedFullGGffgggyyojjhhjkhjyyiuhggjhhjhuyytytyuuytrrtghjtyuggyuighjuyioyyfgffhyuhhghzzrrkkhhkkkhjjhjjjjk}
as
\begin{equation}\label{MaxMedFullGGffgggyyojjhhjkhjyyiuhggjhhjhuyytytyuuytrrtghjtyuggyuighjuyioyyfgffhyuhhghzzrrkkhhkkkhjjhjjjjkjjkjljkl}
\delta(-S)=\frac{1}{\sqrt{\kappa_0\gamma_0}}\left|\partial
\mathcal{S}_0\right|+\frac{2}{c} \left(\vec w\cdot\vec
n\right)\left|\mathcal{S}_0\right|.
\end{equation}
On the other hand, if the light travels across the same curve in the
opposite direction, then we must have:
\begin{equation}\label{MaxMedFullGGffgggyyojjhhjkhjyyiuhggjhhjhuyytytyuuytrrtghjtyuggyuighjuyioyyfgffhyuhhghzzrrkkhhkkkhjjhjjjjkjjkjljklpoo}
\delta(-S^-)=\frac{1}{\sqrt{\kappa_0\gamma_0}}\left|\partial
\mathcal{S}_0\right|-\frac{2}{c} \left(\vec w\cdot\vec
n\right)\left|\mathcal{S}_0\right|.
\end{equation}
Thus, by taking the difference in two cases and using
\er{MaxVacFullPPNmmmffffffiuiuhjuughbghhuiiujjhhjjhjhhjhjjhhjhjjh},
we deduce:
\begin{equation}\label{MaxMedFullGGffgggyyojjhhjkhjyyiuhggjhhjhuyytytyuuytrrtghjtyuggyuighjuyioyyfgffhyuhhghzzrrkkhhkkkhjjhjjjjkjjkjljklkkkhjhj}
\left(\delta(-T)-\delta(-T^-)\right)=k_0\left(\delta(-S)-\delta(-S^-)\right)=\frac{4\omega}{
c^2} \left(\vec w\cdot\vec n\right)\left|\mathcal{S}_0\right|.
\end{equation}
Here, $\gamma_0$ and $\kappa_0$ are the dielectric and the magnetic
permeability of the medium, $T$ is given in
\er{MaxVacFullPPNmmmffffffiuiuhjuughbghhuiiujjhhjjhjhhjhjjhhjhjjh},
$\left|\mathcal{S}_0\right|$ is the area of the flat surface bounded
by the closed path of the light, $\vec n$ is the unit normal to the
surface, $\omega$ is the frequency of the light and $\vec w$ is the
angular speed vector of the rotation of our coordinate system.
%
%
%
%
%
%

\subsubsection{Fizeau experiment}\label{seGOfz}
Assume again the monochromatic electromagnetic wave of the frequency
$\nu=\omega/(2\pi)$ characterized by:
\begin{equation}\label{MaxVacFullPPNmmmffffffiuiuhjuughbghhuiiujjhhjjhjhhjhjjhhjhjjhfz}
U(\vec x,t)=A(\vec x,t)e^{iT(\vec x,t)}=A(\vec x,t)e^{ik_0S(\vec
x,t)},\quad\text{where}\quad
k_0=\frac{\omega}{c}\quad\text{and}\quad\frac{\partial S}{\partial
t}=c\,.
\end{equation}
Then by \er{gughhghfbvnbvyyuuyr} we have
\begin{equation}\label{gughhghfbvnbvyyuuyrhjhjhjfz}
c_0=c\sqrt{\kappa_0\gamma_0}\quad\text{and}\quad\vec {\tilde
u}=\left(\gamma_0\kappa_0\vec v+(1-\gamma_0\kappa_0)\vec u\right),
\end{equation}
where, $\vec u$ is the medium velocity and $\vec v$ is the local
vectorial gravitational potential. Moreover, assume that we consider
light traveling in some region filled with the moving medium of
constant dielectric permeability $\gamma_0$ and magnetic
permeability $\kappa_0$. Then by \er{gughhghfbvnbvyyuuyrhjhjhjfz}
and
\er{MaxMedFullGGffgggyyojjhhjkhjyyiuhggjhhjhuyytytyuuytrrtghjtyuggyuighjuyioyyfgffhyuhhghzzrrkkhhkkkhhhjhkjhhghhgghiuiu1}
we have
\begin{equation}\label{gughhghfbvnbvyyuuyrhiyyuiuuhjhfz}
n=\frac{c}{c_0}=\frac{1}{\sqrt{\kappa_0\gamma_0}}\;\;\;\text{is a
constant,}\quad\text{and}\quad\vec {\tilde u}=\frac{1}{n^2}\,\vec
v+\left(1-\frac{1}{n^2}\right)\vec u.
\end{equation}
Next, assume that the light travels from point $N$ to point $M$
across the curve $\vec r(s):[a,b]\to\mathbb{R}^3$ with endpoints
$\vec r(a)=N$ and $\vec r(b)=M$ undergoing possibly certain number
of reflections from mirrors during its travel. Then, as before, by
\er{MaxMedFullGGffgggyyojjhhjkhjyyiuhggjhhjhuyytytyuuytrrtghjtyuggyuighjuyioyyfgffhyuhhghzzrrkkhhkkk},
\er{gughhghfbvnbvyyuuyrhiyyuiuuhjhfz} and
\er{MaxMedFullGGffgggyyojjyugggjhhjiiuuiyuyuyyuyuzzhhkkk} we have:
\begin{multline}\label{MaxMedFullGGffgggyyojjhhjkhjyyiuhggjhhjhuyytytyuuytrrtghjtyuggyuighjuyioyyfgffhyuhhghzzrrkkhhkkkhjjhfz}
\delta(-S):=\left(-S(M^-)\right)- \left(-S(N^+)\right)=\\
n\int_a^b \left|\vec r'(s)\right|ds-\frac{1}{c}\int_a^b\vec
v\left(\vec r(s)\right)\cdot\vec
r'(s)ds-\frac{n^2}{c}\left(1-\frac{1}{n^2}\right)\int_a^b\vec
u\left(\vec r(s)\right)\cdot\vec r'(s)ds.
\end{multline}
Next assume that, either our curve is perpendicular to the direction
of the vectorial gravitational potential $\vec v$, that happens, for
example, if our path of the light is tangent to the Earth surface,
or assume that our curve is closed, i.e. $M=N$ and moreover, our
coordinate system is inertial, or more generally non-rotating. In
particular, if we assume that $M=N$ i.e. our curve is closed and
moreover, our coordinate system is inertial, or more generally
non-rotating, then, as before, by Stokes Theorem we have:
\begin{equation}\label{MaxMedFullGGffgggyyojjhhjkhjyyiuhggjhhjhuyytytyuuytrrtghjtyuggyuighjuyioyyfgffhyuhhghzzrrkkhhkkkhjjhjjfz}
\int_a^b\vec v\left(\vec r(s)\right)\cdot\vec r'(s)ds=0.
\end{equation}
On the other hand in the case that our curve is perpendicular to the
direction of the vectorial gravitational potential $\vec v$,
\er{MaxMedFullGGffgggyyojjhhjkhjyyiuhggjhhjhuyytytyuuytrrtghjtyuggyuighjuyioyyfgffhyuhhghzzrrkkhhkkkhjjhjjfz}
also trivially follows. Therefore, by inserting
\er{MaxMedFullGGffgggyyojjhhjkhjyyiuhggjhhjhuyytytyuuytrrtghjtyuggyuighjuyioyyfgffhyuhhghzzrrkkhhkkkhjjhjjfz}
into
\er{MaxMedFullGGffgggyyojjhhjkhjyyiuhggjhhjhuyytytyuuytrrtghjtyuggyuighjuyioyyfgffhyuhhghzzrrkkhhkkkhjjhfz}
in both cases we obtain:
\begin{multline}\label{MaxMedFullGGffgggyyojjhhjkhjyyiuhggjhhjhuyytytyuuytrrtghjtyuggyuighjuyioyyfgffhyuhhghzzrrkkhhkkkhjjhfzfzfz}
\delta(-S)=\left(-S(M^-)\right)- \left(-S(N^+)\right)= n\int_a^b
\left|\vec
r'(s)\right|ds-\frac{n^2}{c}\left(1-\frac{1}{n^2}\right)\int_a^b\vec
u\left(\vec r(s)\right)\cdot\vec r'(s)ds.
\end{multline}
Then by
\er{MaxMedFullGGffgggyyojjhhjkhjyyiuhggjhhjhuyytytyuuytrrtghjtyuggyuighjuyioyyfgffhyuhhghzzrrkkhhkkkhjjhfzfzfz}
and
\er{MaxVacFullPPNmmmffffffiuiuhjuughbghhuiiujjhhjjhjhhjhjjhhjhjjhfz}
we deduce
\begin{multline}\label{MaxMedFullGGffgggyyojjhhjkhjyyiuhggjhhjhuyytytyuuytrrtghjtyuggyuighjuyioyyfgffhyuhhghzzrrkkhhkkkhjjhfzfzfzfzffz}
\delta(-T):=\left(-T(M^-)\right)-
\left(-T(N^+)\right)=k_0\delta(-S)\\= \frac{n\omega}{c}\int_a^b
\left|\vec
r'(s)\right|ds-\frac{n^2\omega}{c^2}\left(1-\frac{1}{n^2}\right)\int_a^b\vec
u\left(\vec r(s)\right)\cdot\vec r'(s)ds\\
\frac{n^2\omega}{c^2}\left(c_0\int_a^b \left|\vec
r'(s)\right|ds-\left(1-\frac{1}{n^2}\right)\int_a^b\vec u\left(\vec
r(s)\right)\cdot\vec r'(s)ds\right).
\end{multline}
In particular, if the absolute value $\left|\vec u\left(\vec
r(s)\right)\right|$ is a constant across the curve and if the angle
between $\vec r'(s)$ and $\vec u\left(\vec r(s)\right)$ is a
constant across the curve and equals to the value $\theta$ then
denoting the length of the path by $L$:
\begin{equation}\label{MaxMedFullGGffgggyyojjhhjkhjyyiuhggjhhjhuyytytyuuytrrtghjtyuggyuighjuyioyyfgffhyuhhghzzrrkkhhkkkhjjhjjfzuyy}
L:=\int_a^b \left|\vec r'(s)\right|ds,
\end{equation}
by
\er{MaxMedFullGGffgggyyojjhhjkhjyyiuhggjhhjhuyytytyuuytrrtghjtyuggyuighjuyioyyfgffhyuhhghzzrrkkhhkkkhjjhfzfzfzfzffz}
we deduce:
\begin{equation}\label{MaxMedFullGGffgggyyojjhhjkhjyyiuhggjhhjhuyytytyuuytrrtghjtyuggyuighjuyioyyfgffhyuhhghzzrrkkhhkkkhjjhfzfzfzfzffzghghgfz}
\delta(-T)=k_0\delta(-S)=\frac{\omega L
n^2}{c^2}\left(c_0-\left(1-\frac{1}{n^2}\right)|\vec
u|\cos{(\theta)}\right).
\end{equation}
Thus, if the direction of $\vec u$ coincides with the direction of
the light i.e. $\theta=0$ then
\begin{equation}\label{MaxMedFullGGffgggyyojjhhjkhjyyiuhggjhhjhuyytytyuuytrrtghjtyuggyuighjuyioyyfgffhyuhhghzzrrkkhhkkkhjjhfzfzfzfzffzghghgfzytyfz}
\delta(-T)=k_0\delta(-S)=\frac{\omega L
n^2}{c^2}\left(c_0-\left(1-\frac{1}{n^2}\right)|\vec
u|\right)\approx\frac{\omega
L}{\left(c_0+\left(1-\frac{1}{n^2}\right)|\vec u|\right)}.
\end{equation}
On the other hand, if the direction of $\vec u$ is opposite to the
direction of the light i.e. $\theta=\pi$ then
\begin{equation}\label{MaxMedFullGGffgggyyojjhhjkhjyyiuhggjhhjhuyytytyuuytrrtghjtyuggyuighjuyioyyfgffhyuhhghzzrrkkhhkkkhjjhfzfzfzfzffzghghgfzuyuyhffz}
\delta(-T)=k_0\delta(-S)=\frac{\omega L
n^2}{c^2}\left(c_0+\left(1-\frac{1}{n^2}\right)|\vec
u|\right)\approx\frac{\omega
L}{\left(c_0-\left(1-\frac{1}{n^2}\right)|\vec u|\right)}.
\end{equation}
So, in the frames of our model we explain the results of the Fizeau
experiment.

\subsubsection{Fermat Principle of Geometric Optics in the case when we cannot neglect effects of order $O\left(\frac{|\vec {\tilde
u}|^2}{c^2_0}\right)$}
Consider again a monochromatic wave of the
frequency $\nu=\omega/(2\pi)$ characterized by:
\begin{equation}\label{MaxVacFullPPNmmmffffffiuiuhjuughbghhuiiujjhhjjhjhhjyuyhh}
U(\vec x,t)=A(\vec x,t)e^{ik_0S(\vec x,t)},\quad\text{where}\quad
k_0=\frac{\omega}{c}\quad\text{and}\quad\frac{\partial S}{\partial
t}=c\,.
\end{equation}
As before, assume the validity of the rough Geometric Optics
approximation. However, assume that we cannot consider anymore the
case of the approximation, given by \er{ojhkk}, i.e. we assume that
we cannot neglect anymore effects of order $O\left(\frac{|\vec
{\tilde u}|^2}{c^2_0}\right)$. This happens, for example in the case
of the Michelson-Morley experiment. Thus instead of
\er{MaxVacFullPPNmmmffffffiuiuhjuughbghhiuijghghghhjhjhhghyuyiyyujjljk}
and
\er{MaxVacFullPPNmmmffffffiuiuhjuughbghhiuijghghghhhfhhghghguygtjuuujjjk}
we need to deal with
\er{MaxVacFullPPNmmmffffffiuiuhjuughbghhiuijghghghhjhjhhghyuyiyyujj}
 and
\er{MaxVacFullPPNmmmffffffiuiuhjuughbghhiuijghghghhhfhhghghguygtjuuujj}.
On the other hand, by
\er{MaxVacFullPPNmmmffffffiuiuhjuughbghhiuijghghghhjhjhhghyuyiyyujj}
we deduce:
\begin{multline}\label{MaxVacFullPPNmmmffffffiuiuhjuughbghhiuijghghghhjhjhhghyuyiyyujjkhgggghgg}
\left|\nabla_{\vec x}S-\frac{c}{c_0}\left(1+\frac{1}{c}\left(\vec
{\tilde u}\cdot\nabla_{\vec x}S\right)\right)\frac{\vec {\tilde
u}}{c_0}\right|^2=\frac{c^2}{c^2_0}\left(1+\frac{1}{c^2}\left|\nabla_\vec
x S\right|^2\left|\vec {\tilde u}\right|^2-\frac{1}{c^2}\left|\vec
{\tilde u}\cdot\nabla_\vec x S\right|^2\right)=\\
\frac{c^2}{c^2_0}\left(1+\frac{1}{c^2}\left|\nabla_{\vec
x}S-\frac{c}{c_0}\left(1+\frac{1}{c}\left(\vec {\tilde
u}\cdot\nabla_{\vec x}S\right)\right)\frac{\vec {\tilde
u}}{c_0}\right|^2\left|\vec {\tilde
u}\right|^2-\frac{1}{c^2}\left|\vec {\tilde
u}\cdot\left(\nabla_{\vec
x}S-\frac{c}{c_0}\left(1+\frac{1}{c}\left(\vec {\tilde
u}\cdot\nabla_{\vec x}S\right)\right)\frac{\vec {\tilde
u}}{c_0}\right)\right|^2\right).
\end{multline}
Then we rewrite
\er{MaxVacFullPPNmmmffffffiuiuhjuughbghhiuijghghghhhfhhghghguygtjuuujj}
and
\er{MaxVacFullPPNmmmffffffiuiuhjuughbghhiuijghghghhjhjhhghyuyiyyujjkhgggghgg}
as:
\begin{equation}\label{MaxVacFullPPNmmmffffffiuiuhjuughbghhiuijghghghhhfhhghghguygtjuuujjjhy}
\vec h_0\cdot\nabla_{\vec x}A=\frac{1}{2}\left(\left(\Delta_{\vec
x}S\right)-\Div_{\vec x}\left\{\frac{1}{c^2_0}\left(\vec {\tilde
u}\cdot\nabla_{\vec x}S\right)\vec {\tilde u}\right\}\right)A.
\end{equation}
and
\begin{equation}\label{MaxVacFullPPNmmmffffffiuiuhjuughbghhiuijghghghhjhjhhghyuyiyyujjkhgggghggutgg}
\left(1-\frac{\left|\vec {\tilde
u}\right|^2}{c^2_0}\right)\left|\vec h-\frac{1}{\left|\vec {\tilde
u}\right|^2}\left(\vec {\tilde u}\cdot\vec h\right)\vec {\tilde
u}\right|^2+\left|\frac{1}{\left|\vec {\tilde u}\right|}\vec {\tilde
u}\cdot\vec h\right|^2=\left|\vec h\right|^2\left(1-\frac{\left|\vec
{\tilde u}\right|^2}{c^2_0}\right)+\frac{1}{c^2_0}\left|\vec {\tilde
u}\cdot\vec h\right|^2= \frac{c^2}{c^2_0},
\end{equation}
where the vector field $\vec h_0$ defined for every $\vec x$ by:
\begin{equation}\label{MaxVacFullPPNmmmffffffiuiuhjuughbghhiuijghghghhhfhhghghguygtjuuujjjkyuuyhh}
\vec h_0(\vec x):=\frac{c}{c^2_0(\vec
x)}\left(1+\frac{1}{c}\left(\vec {\tilde u(\vec x)}\cdot\nabla_{\vec
x}S(\vec x)\right)\right)\vec {\tilde u}(\vec x)-\nabla_{\vec
x}S(\vec x) =\frac{c}{c^2_0(\vec x)}\left(1+\frac{1}{c}\left(\vec
{\tilde u(\vec x)}\cdot\nabla_{\vec x}S(\vec x)\right)\right)\vec
h(\vec x) ,
\end{equation}
is called the direction of propagation of the ray that passes
through point $\vec x$. We clarify this name bellow.
The Eikonal equation
\er{MaxVacFullPPNmmmffffffiuiuhjuughbghhiuijghghghhjhjhhghyuyiyyujjkhgggghggutgg}
and equation of the beam propagation
\er{MaxVacFullPPNmmmffffffiuiuhjuughbghhiuijghghghhhfhhghghguygtjuuujjjhy}
are two basic equations of propagation of monochromatic light in the
Geometric Optics approximation inside a moving medium or/and in the
presence of non-trivial gravitational field.

Next if we consider an arbitrary characteristic curve $\vec
r(s):[a,b]\to\mathbb{R}^3$ of equation
\er{MaxVacFullPPNmmmffffffiuiuhjuughbghhiuijghghghhhfhhghghguygtjuuujjjhy}
defined as a solution of ordinary differential equation
\begin{equation}\label{MaxVacFullPPNmmmffffffiuiuhjuughbghhiuijghghghhhfhhghghguygtjuuujjjkkhh}
\begin{cases}
\frac{d\vec r}{ds}(s)=\vec h_0\left(\vec
r(s)\right)\\
\vec r(a)=\vec x_0,
\end{cases}
\end{equation}
then by
\er{MaxVacFullPPNmmmffffffiuiuhjuughbghhiuijghghghhhfhhghghguygtjuuujjjhy}
and
\er{MaxVacFullPPNmmmffffffiuiuhjuughbghhiuijghghghhhfhhghghguygtjuuujjjkkhh}
we have
\begin{equation}\label{MaxVacFullPPNmmmffffffiuiuhjuughbghhiuijghghghhhfhhghghguygtjuuujjjkhjhjhhh}
\frac{d}{ds}\left(A\left(\vec r(s)\right)\right)=\nabla_{\vec
x}A\left(\vec r(s)\right)\cdot\frac{d\vec r}{ds}(s)= \frac{1}{2}
g\left(\vec r(s)\right)A\left(\vec r(s)\right),
\end{equation}
where we denote
\begin{equation}\label{MaxVacFullPPNmmmffffffiuiuhjuughbghhiuijghghghhhfhhghghguygtjuuujjjkhjhjhffgghhpiuuu}
g(\vec x):=\Delta_{\vec x}S(\vec x)-\frac{1}{c^2_0(\vec
x)}\left(\left(\nabla^2_{\vec x}S(\vec x)\cdot\vec {\tilde u}(\vec
x)\right)\cdot\vec {\tilde u}(\vec x)\right).
\end{equation}
Then
\er{MaxVacFullPPNmmmffffffiuiuhjuughbghhiuijghghghhhfhhghghguygtjuuujjjkhjhjhhh}
implies
\begin{equation}\label{MaxVacFullPPNmmmffffffiuiuhjuughbghhiuijghghghhhfhhghghguygtjuuujjjkhjhjhffgghh}
A\left(\vec r(s)\right)=A\left(\vec
x_0\right)e^{\frac{1}{2}\int_a^{s}g\left(\vec
r(\tau)\right)d\tau}\quad\quad\forall s\in[a,b].
\end{equation}
In particular,
\begin{equation}\label{MaxVacFullPPNmmmffffffiuiuhjuughbghhiuijghghghhhfhhghghguygtjuuujjjkhjhjhffggiouuiiuhh}
A\left(\vec x_0\right)=0\;\;\text{implies}\;\; A\left(\vec
r(s)\right)=0\quad\forall s\in[a,b],\quad\text{and}\quad A\left(\vec
x_0\right)\neq 0\;\;\text{implies}\;\; A\left(\vec r(s)\right)\neq
0\quad\forall s\in[a,b].
\end{equation}
Therefore, by
\er{MaxVacFullPPNmmmffffffiuiuhjuughbghhiuijghghghhhfhhghghguygtjuuujjjkhjhjhffggiouuiiuhh}
we deduce that the curve that satisfies
\er{MaxVacFullPPNmmmffffffiuiuhjuughbghhiuijghghghhhfhhghghguygtjuuujjjkkhh}
coincides with the ray of light that passes through the point $\vec
x_0$. So
\er{MaxVacFullPPNmmmffffffiuiuhjuughbghhiuijghghghhhfhhghghguygtjuuujjjkkhh}
is the equation of a ray and the vector field $\vec h_0$ defined for
every $\vec x$ by
\er{MaxVacFullPPNmmmffffffiuiuhjuughbghhiuijghghghhhfhhghghguygtjuuujjjkyuuyhh}
is indeed the direction of propagation of the ray that passes
through point $\vec x$.
%
%
%
\begin{remark}\label{hjgjgghjj11}
As before in remark \ref{hjgjggh}, in contrast to the proof of
\er{MaxVacFullPPNmmmffffffiuiuhjuughbghhiuijghghghhjhjhhghyuyiyyujj}
or
\er{MaxVacFullPPNmmmffffffiuiuhjuughbghhiuijghghghhjhjhhghyuyiyyujjkhgggghgg},
we do not use the Geometric Optics approximation
\er{slochangGGaaffgfgjhjggh} in the proof of
\er{MaxVacFullPPNmmmffffffiuiuhjuughbghhiuijghghghhhfhhghghguygtjuuujj}
and
\er{MaxVacFullPPNmmmffffffiuiuhjuughbghhiuijghghghhhfhhghghguygtjuuujjjhy}.
So,
\er{MaxVacFullPPNmmmffffffiuiuhjuughbghhiuijghghghhhfhhghghguygtjuuujjjkhjhjhffggiouuiiuhh}
is still valid without assumption of the Geometric Optics
approximation \er{slochangGGaaffgfgjhjggh} and thus, the vector
field $\vec h_0$, defined by
\er{MaxVacFullPPNmmmffffffiuiuhjuughbghhiuijghghghhhfhhghghguygtjuuujjjkyuuyhh},
is the direction of the propagation of the ray that passes through
point $\vec x$ also in the general case. However, without assumption
of the Geometric Optics approximation we cannot derive
\er{MaxVacFullPPNmmmffffffiuiuhjuughbghhiuijghghghhjhjhhghyuyiyyujjkhgggghggutgg}
anymore.
\end{remark}
Next, by
\er{MaxVacFullPPNmmmffffffiuiuhjuughbghhiuijghghghhhfhhghghguygtjuuujjjkyuuyhh}
we have:
\begin{equation}\label{MaxVacFullPPNmmmffffffiuiuhjuughbghhiuijghghghhhfhhghghguygtjuuujjjkyuuyhhijhhhhjihhyjh}
\left(1-\frac{|\vec {\tilde
u}|^2}{c^2_0}\right)^{-1}\left(\frac{1}{|\vec {\tilde u}|}\vec
{\tilde u}\cdot\vec h_0\right)=\frac{c}{c^2_0}\left(1-\frac{|\vec
{\tilde u}|^2}{c^2_0}\right)^{-1}|\vec {\tilde
u}|-\left(\frac{1}{|\vec {\tilde u}|}\vec {\tilde
u}\cdot\nabla_{\vec x}S\right),
\end{equation}
and
\begin{equation}\label{MaxVacFullPPNmmmffffffiuiuhjuughbghhiuijghghghhhfhhghghguygtjuuujjjkyuuyhhijhhhhjjhjhj}
\vec h_0-\frac{1}{\left|\vec {\tilde u}\right|^2}\left(\vec {\tilde
u}\cdot\vec h_0\right)\vec {\tilde u}=-\left(\nabla_{\vec
x}S-\frac{1}{\left|\vec {\tilde u}\right|^2}\left(\vec {\tilde
u}\cdot\nabla_{\vec x}S\right)\vec {\tilde u}\right),
\end{equation}
On the other hand by
\er{MaxVacFullPPNmmmffffffiuiuhjuughbghhiuijghghghhjhjhhghyuyiyyujjkhgggghggutgg}
we have
\begin{equation}\label{MaxVacFullPPNmmmffffffiuiuhjuughbghhiuijghghghhjhjhhghyuyiyyujjkhgggghggutggjkkjkll}
\left|\vec h_0-\frac{1}{\left|\vec {\tilde u}\right|^2}\left(\vec
{\tilde u}\cdot\vec h_0\right)\vec {\tilde
u}\right|^2+\left(1-\frac{\left|\vec {\tilde
u}\right|^2}{c^2_0}\right)^{-1}\left|\frac{1}{\left|\vec {\tilde
u}\right|}\vec {\tilde u}\cdot\vec h_0\right|^2=
\frac{c^2}{c^2_0}\left(1-\frac{\left|\vec {\tilde
u}\right|^2}{c^2_0}\right)^{-1},
\end{equation}
Therefore if we consider a curve $\vec r(s):[a,b]\to\mathbb{R}^3$
with endpoints $\vec r(a)=N$ and $\vec r(b)=M$, then integrating the
square root of both sides of
\er{MaxVacFullPPNmmmffffffiuiuhjuughbghhiuijghghghhjhjhhghyuyiyyujjkhgggghggutggjkkjkll}
over the curve $\vec r(s)$ we deduce
\begin{multline}\label{MaxMedFullGGffgggyyojjhhjkhjyyiuhggjhhjhuyytytyuuytrrtghjtyuggyuighjuyioyyfgffhyuhhghzzrrhhkkkhh}
\int\limits_a^b\sqrt{\left|\vec h_0\left(\vec r(s)\right)-\left(\vec
{\tilde u}\left(\vec r(s)\right)\cdot\vec h_0\left(\vec
r(s)\right)\right)\frac{\vec {\tilde u}\left(\vec
r(s)\right)}{\left|\vec {\tilde u}\left(\vec
r(s)\right)\right|^2}\right|^2+\left(1-\frac{\left|\vec {\tilde
u}\left(\vec r(s)\right)\right|^2}{c^2_0\left(\vec
r(s)\right)}\right)^{-1}\left|\frac{\vec {\tilde u}\left(\vec
r(s)\right)}{\left|\vec {\tilde u}\left(\vec
r(s)\right)\right|}\cdot\vec h_0\left(\vec r(s)\right)\right|^2}
\cdot\\
\cdot\sqrt{\left|\vec r'(s)-\left(\vec {\tilde u}\left(\vec
r(s)\right)\cdot\vec r'(s)\right)\frac{\vec {\tilde u}\left(\vec
r(s)\right)}{\left|\vec {\tilde u}\left(\vec
r(s)\right)\right|^2}\right|^2+\left(1-\frac{\left|\vec {\tilde
u}\left(\vec r(s)\right)\right|^2}{c^2_0\left(\vec
r(s)\right)}\right)^{-1}\left|\frac{\vec {\tilde u}\left(\vec
r(s)\right)}{\left|\vec {\tilde u}\left(\vec
r(s)\right)\right|}\cdot\vec r'(s)\right|^2}ds
\\=\int\limits_a^b\frac{c}{c_0\left(\vec
r(s)\right)}\left(1-\frac{\left|\vec {\tilde u}\left(\vec
r(s)\right)\right|^2}{c^2_0\left(\vec
r(s)\right)}\right)^{-\frac{1}{2}}
\cdot\\
\cdot\sqrt{\left|\vec r'(s)-\left(\vec {\tilde u}\left(\vec
r(s)\right)\cdot\vec r'(s)\right)\frac{\vec {\tilde u}\left(\vec
r(s)\right)}{\left|\vec {\tilde u}\left(\vec
r(s)\right)\right|^2}\right|^2+\left(1-\frac{\left|\vec {\tilde
u}\left(\vec r(s)\right)\right|^2}{c^2_0\left(\vec
r(s)\right)}\right)^{-1}\left|\frac{\vec {\tilde u}\left(\vec
r(s)\right)}{\left|\vec {\tilde u}\left(\vec
r(s)\right)\right|}\cdot\vec r'(s)\right|^2}ds.
\end{multline}
Thus in particular, by inserting
\er{MaxVacFullPPNmmmffffffiuiuhjuughbghhiuijghghghhhfhhghghguygtjuuujjjkyuuyhhijhhhhjihhyjh}
and
\er{MaxVacFullPPNmmmffffffiuiuhjuughbghhiuijghghghhhfhhghghguygtjuuujjjkyuuyhhijhhhhjjhjhj}
into
\er{MaxMedFullGGffgggyyojjhhjkhjyyiuhggjhhjhuyytytyuuytrrtghjtyuggyuighjuyioyyfgffhyuhhghzzrrhhkkkhh}
and using inequality $\vec a\cdot\vec b\leq |\vec a||\vec b|$, we
deduce
\begin{multline}\label{MaxMedFullGGffgggyyojjhhjkhjyyiuhggjhhjhuyytytyuuytrrtghjtyuggyuighjuyioyyfgffhyuhhghzzrriuihhkkkhh}
\int_a^b\frac{c}{c^2_0\left(\vec
r(s)\right)}\left(1-\frac{\left|\vec {\tilde u}\left(\vec
r(s)\right)\right|^2}{c^2_0}\right)^{-1}\vec {\tilde u}\left(\vec
r(s)\right)\cdot\vec r'(s)ds-\int_a^b\nabla_{\vec x}S\left(\vec
r(s)\right)\cdot\vec r'(s)ds\\
\leq\int\limits_a^b\frac{c}{c_0\left(\vec
r(s)\right)}\left(1-\frac{\left|\vec {\tilde u}\left(\vec
r(s)\right)\right|^2}{c^2_0\left(\vec
r(s)\right)}\right)^{-\frac{1}{2}}
\cdot\\
\cdot\sqrt{\left|\vec r'(s)-\left(\vec {\tilde u}\left(\vec
r(s)\right)\cdot\vec r'(s)\right)\frac{\vec {\tilde u}\left(\vec
r(s)\right)}{\left|\vec {\tilde u}\left(\vec
r(s)\right)\right|^2}\right|^2+\left(1-\frac{\left|\vec {\tilde
u}\left(\vec r(s)\right)\right|^2}{c^2_0\left(\vec
r(s)\right)}\right)^{-1}\left|\frac{\vec {\tilde u}\left(\vec
r(s)\right)}{\left|\vec {\tilde u}\left(\vec
r(s)\right)\right|}\cdot\vec r'(s)\right|^2}ds,
\end{multline}
i.e.
\begin{multline}\label{MaxMedFullGGffgggyyojjhhjkhjyyiuhggjhhjhuyytytyuuytrrtghjtyuggyuighjuyioyyfgffhyuhhghzzrrkkijjhjhhkkkhh}
\left(-S(M)\right)-
\left(-S(N)\right)\leq-\int_a^b\frac{c}{c^2_0\left(\vec
r(s)\right)}\left(1-\frac{\left|\vec {\tilde u}\left(\vec
r(s)\right)\right|^2}{c^2_0\left(\vec r(s)\right)}\right)^{-1}\vec
{\tilde u}\left(\vec r(s)\right)\cdot\vec
r'(s)ds\\+\int\limits_a^b\frac{c}{c_0\left(\vec
r(s)\right)}\left(1-\frac{\left|\vec {\tilde u}\left(\vec
r(s)\right)\right|^2}{c^2_0\left(\vec
r(s)\right)}\right)^{-\frac{1}{2}}
\cdot\\
\cdot\sqrt{\left|\vec r'(s)-\left(\vec {\tilde u}\left(\vec
r(s)\right)\cdot\vec r'(s)\right)\frac{\vec {\tilde u}\left(\vec
r(s)\right)}{\left|\vec {\tilde u}\left(\vec
r(s)\right)\right|^2}\right|^2+\left(1-\frac{\left|\vec {\tilde
u}\left(\vec r(s)\right)\right|^2}{c^2_0\left(\vec
r(s)\right)}\right)^{-1}\left|\frac{\vec {\tilde u}\left(\vec
r(s)\right)}{\left|\vec {\tilde u}\left(\vec
r(s)\right)\right|}\cdot\vec r'(s)\right|^2}ds.
\end{multline}
Moreover, if
\begin{equation}\label{MaxMedFullGGffgggyyojjyugggjhhjzzrrhhkkkhh}
\frac{d\vec r}{ds}(s)=\sigma(s)\vec h_0\left(\vec r(s)\right),
\end{equation}
for some nonnegative scalar factor $\sigma=\sigma(s)$ then by
\er{MaxMedFullGGffgggyyojjyugggjhhjzzrrhhkkkhh}, exactly in the same
way as we get
\er{MaxMedFullGGffgggyyojjhhjkhjyyiuhggjhhjhuyytytyuuytrrtghjtyuggyuighjuyioyyfgffhyuhhghzzrriuihhkkkhh},
we rewrite
\er{MaxMedFullGGffgggyyojjhhjkhjyyiuhggjhhjhuyytytyuuytrrtghjtyuggyuighjuyioyyfgffhyuhhghzzrrhhkkkhh}
as
\begin{multline}\label{MaxMedFullGGffgggyyojjhhjkhjyyiuhggjhhjhuyytytyuuytrrtghjtyuggyuighjuyioyyfgffhyuhhghzzrrkkhhkkkhh}
\left(-S(M)\right)-
\left(-S(N)\right)=-\int_a^b\frac{c}{c^2_0\left(\vec
r(s)\right)}\left(1-\frac{\left|\vec {\tilde u}\left(\vec
r(s)\right)\right|^2}{c^2_0\left(\vec r(s)\right)}\right)^{-1}\vec
{\tilde u}\left(\vec r(s)\right)\cdot\vec
r'(s)ds\\+\int\limits_a^b\frac{c}{c_0\left(\vec
r(s)\right)}\left(1-\frac{\left|\vec {\tilde u}\left(\vec
r(s)\right)\right|^2}{c^2_0\left(\vec
r(s)\right)}\right)^{-\frac{1}{2}}
\cdot\\
\cdot\sqrt{\left|\vec r'(s)-\left(\vec {\tilde u}\left(\vec
r(s)\right)\cdot\vec r'(s)\right)\frac{\vec {\tilde u}\left(\vec
r(s)\right)}{\left|\vec {\tilde u}\left(\vec
r(s)\right)\right|^2}\right|^2+\left(1-\frac{\left|\vec {\tilde
u}\left(\vec r(s)\right)\right|^2}{c^2_0\left(\vec
r(s)\right)}\right)^{-1}\left|\frac{\vec {\tilde u}\left(\vec
r(s)\right)}{\left|\vec {\tilde u}\left(\vec
r(s)\right)\right|}\cdot\vec r'(s)\right|^2}ds.
\end{multline}
Thus, by comparing
\er{MaxVacFullPPNmmmffffffiuiuhjuughbghhiuijghghghhhfhhghghguygtjuuujjjkkhh}
with \er{MaxMedFullGGffgggyyojjyugggjhhjzzrrhhkkkhh} and using
\er{MaxMedFullGGffgggyyojjhhjkhjyyiuhggjhhjhuyytytyuuytrrtghjtyuggyuighjuyioyyfgffhyuhhghzzrrkkijjhjhhkkkhh}
and
\er{MaxMedFullGGffgggyyojjhhjkhjyyiuhggjhhjhuyytytyuuytrrtghjtyuggyuighjuyioyyfgffhyuhhghzzrrkkhhkkkhh},
we deduce that if we assume that the light travel from the point $N$
to the point $M$ across the curve $\vec {\tilde
r}(s):[a,b]\to\mathbb{R}^3$ such that $\vec{\tilde r(a)}=N$ and
$\vec {\tilde r}(b)=M$, then
\begin{multline}\label{MaxMedFullGGffgggyyojjhhjkhjyyiuhggjhhjhuyytytyuuytrrtghjtyuggyuighjuyioyyfgffhyuhhghzzrrkkhhkkkhhhhh}
\left(-S(M)\right)-
\left(-S(N)\right)=-\int_a^b\frac{c}{c^2_0\left(\vec{\tilde
r}(s)\right)}\left(1-\frac{\left|\vec {\tilde u}\left(\vec{\tilde
r}(s)\right)\right|^2}{c^2_0\left(\vec{\tilde
r}(s)\right)}\right)^{-1}\vec {\tilde u}\left(\vec{\tilde
r}(s)\right)\cdot\vec
r'(s)ds\\+\int\limits_a^b\frac{c}{c_0\left(\vec{\tilde
r}(s)\right)}\left(1-\frac{\left|\vec {\tilde u}\left(\vec{\tilde
r}(s)\right)\right|^2}{c^2_0\left(\vec{\tilde
r}(s)\right)}\right)^{-\frac{1}{2}}
\cdot\\
\cdot\sqrt{\left|\vec{\tilde r}'(s)-\left(\vec {\tilde
u}\left(\vec{\tilde r}(s)\right)\cdot\vec{\tilde
r}'(s)\right)\frac{\vec {\tilde u}\left(\vec{\tilde
r}(s)\right)}{\left|\vec {\tilde u}\left(\vec{\tilde
r}(s)\right)\right|^2}\right|^2+\left(1-\frac{\left|\vec {\tilde
u}\left(\vec{\tilde r}(s)\right)\right|^2}{c^2_0\left(\vec{\tilde
r}(s)\right)}\right)^{-1}\left|\frac{\vec {\tilde
u}\left(\vec{\tilde r}(s)\right)}{\left|\vec {\tilde
u}\left(\vec{\tilde r}(s)\right)\right|}\cdot\vec{\tilde
r}'(s)\right|^2}ds,
\end{multline}
and for every other curve $\vec r(s):[a,b]\to\mathbb{R}^3$ with
endpoints $\vec r(a)=N$ and $\vec r(b)=M$ we have
\begin{multline}\label{MaxMedFullGGffgggyyojjhhjkhjyyiuhggjhhjhuyytytyuuytrrtghjtyuggyuighjuyioyyfgffhyuhhghzzrrkkhhkkkhhhjhkjhhhh}
-\int_a^b\frac{c}{c^2_0\left(\vec
r(s)\right)}\left(1-\frac{\left|\vec {\tilde u}\left(\vec
r(s)\right)\right|^2}{c^2_0\left(\vec r(s)\right)}\right)^{-1}\vec
{\tilde u}\left(\vec r(s)\right)\cdot\vec
r'(s)ds\\+\int\limits_a^b\frac{c}{c_0\left(\vec
r(s)\right)}\left(1-\frac{\left|\vec {\tilde u}\left(\vec
r(s)\right)\right|^2}{c^2_0\left(\vec
r(s)\right)}\right)^{-\frac{1}{2}}
\cdot\\
\cdot\sqrt{\left|\vec r'(s)-\left(\vec {\tilde u}\left(\vec
r(s)\right)\cdot\vec r'(s)\right)\frac{\vec {\tilde u}\left(\vec
r(s)\right)}{\left|\vec {\tilde u}\left(\vec
r(s)\right)\right|^2}\right|^2+\left(1-\frac{\left|\vec {\tilde
u}\left(\vec r(s)\right)\right|^2}{c^2_0\left(\vec
r(s)\right)}\right)^{-1}\left|\frac{\vec {\tilde u}\left(\vec
r(s)\right)}{\left|\vec {\tilde u}\left(\vec
r(s)\right)\right|}\cdot\vec r'(s)\right|^2}ds\\ \geq
-\int_a^b\frac{c}{c^2_0\left(\vec{\tilde
r}(s)\right)}\left(1-\frac{\left|\vec {\tilde u}\left(\vec{\tilde
r}(s)\right)\right|^2}{c^2_0\left(\vec{\tilde
r}(s)\right)}\right)^{-1}\vec {\tilde u}\left(\vec{\tilde
r}(s)\right)\cdot\vec
r'(s)ds\\+\int\limits_a^b\frac{c}{c_0\left(\vec{\tilde
r}(s)\right)}\left(1-\frac{\left|\vec {\tilde u}\left(\vec{\tilde
r}(s)\right)\right|^2}{c^2_0\left(\vec{\tilde
r}(s)\right)}\right)^{-\frac{1}{2}}
\cdot\\
\cdot\sqrt{\left|\vec{\tilde r}'(s)-\left(\vec {\tilde
u}\left(\vec{\tilde r}(s)\right)\cdot\vec{\tilde
r}'(s)\right)\frac{\vec {\tilde u}\left(\vec{\tilde
r}(s)\right)}{\left|\vec {\tilde u}\left(\vec{\tilde
r}(s)\right)\right|^2}\right|^2+\left(1-\frac{\left|\vec {\tilde
u}\left(\vec{\tilde r}(s)\right)\right|^2}{c^2_0\left(\vec{\tilde
r}(s)\right)}\right)^{-1}\left|\frac{\vec {\tilde
u}\left(\vec{\tilde r}(s)\right)}{\left|\vec {\tilde
u}\left(\vec{\tilde r}(s)\right)\right|}\cdot\vec{\tilde
r}'(s)\right|^2}ds.
\end{multline}
So we obtain the following Fermat Principle:
\begin{proposition}\label{gughghfhh}
Assume Geometric Optics approximation.
Then
the light that travels from point $N$ to point $M$ chooses the path
$\vec r(s):[a,b]\to\mathbb{R}^3$ with endpoints $\vec r(a)=N$ and
$\vec r(b)=M$ which minimizes the quantity:
\begin{multline}\label{MaxMedFullGGffgggyyojjhhjkhjyyiuhggjhhjhuyytytyuuytrrtghjtyuggyuighjuyioyyfgffhyuhhghzzrrkkhhkkkhhhjhkjhhghhgghhh}
J\left(\vec r(\cdot)\right):= -\int_a^b\frac{1}{c}n^2\left(\vec
r(s)\right)\left(1-\frac{\left|\vec {\tilde u}\left(\vec
r(s)\right)\right|^2}{c^2_0\left(\vec r(s)\right)}\right)^{-1}\vec
{\tilde u}\left(\vec r(s)\right)\cdot\vec
r'(s)ds\\+\int\limits_a^bn\left(\vec
r(s)\right)\left(1-\frac{\left|\vec {\tilde u}\left(\vec
r(s)\right)\right|^2}{c^2_0\left(\vec r(s)\right)}\right)^{-1}
\cdot\\
\cdot\sqrt{\left(1-\frac{\left|\vec {\tilde u}\left(\vec
r(s)\right)\right|^2}{c^2_0\left(\vec r(s)\right)}\right)\left|\vec
r'(s)-\left(\vec {\tilde u}\left(\vec r(s)\right)\cdot\vec
r'(s)\right)\frac{\vec {\tilde u}\left(\vec r(s)\right)}{\left|\vec
{\tilde u}\left(\vec r(s)\right)\right|^2}\right|^2+\left|\frac{\vec
{\tilde u}\left(\vec r(s)\right)}{\left|\vec {\tilde u}\left(\vec
r(s)\right)\right|}\cdot\vec r'(s)\right|^2}ds,
\end{multline}
where we set refraction index:
\begin{equation}\label{MaxMedFullGGffgggyyojjhhjkhjyyiuhggjhhjhuyytytyuuytrrtghjtyuggyuighjuyioyyfgffhyuhhghzzrrkkhhkkkhhhjhkjhhghhgghiuiuhh}
n\left(\vec x\right):=\frac{c}{c_0\left(\vec x\right)}.
\end{equation}
\end{proposition}

\subsubsection{The case of the non-monochromatic wave or/and  moving domains of propagation of light}\label{mnGObnbn}
Next, assume that the wave is not monochromatic or/and the fields
$\vec {\tilde u}$ and $c_0$ depend on the time variable or/and we
consider the case of moving domains of propagation of light (in
particular moving surfaces of reflection/refraction). In other words
we can not assume
\er{MaxVacFullPPNmmmffffffiuiuhjuughbghhiuijghghghhjhjhhghyuyhhjhjihikjjjjkjljkl}
or
\er{MaxVacFullPPNmmmffffffiuiuhjuughbghhiuijghghghhjhjhhghyuyhhjhjihikjjj}
anymore. However, we do assume the rough Geometric Optics
approximation \er{slochangGGaaffgfgjhjggh11}. We also assume that
our medium has no dispersion. Then, due to
\er{MaxVacFullPPNmmmffffffiuiuhjuughbghhuiiujj} we have:
\begin{equation}\label{MaxVacFullPPNmmmffffffiuiuhjuughbghhuiiujjhkhkh}
U(\vec x,t)=A(\vec x,t)e^{ik_0S(\vec x,t)},
\end{equation}
and by \er{MaxVacFullPPNmmmffffffiuiuhjuughbghhuiiujjjjjjjjhhhjjj}
and \er{MaxVacFullPPNmmmffffffiuiuhjuughbghhuiiujjjjjjjjhhhjjjkk} we
have:
\begin{equation}\label{MaxVacFullPPNmmmffffffiuiuhjuughbghhuiiujjjjjjjjhhhjjjjkkjjkihhhj}
\left<\,\left|\frac{\partial S}{\partial t}\right|\,\right>=c,
\end{equation}
where the sign $\left<\cdot\right>$ means the spatial and temporal
averaging. Then, due to
\er{MaxVacFullPPNmmmffffffiuiuhjuughbghhiuijghghghhjhjhhghyuyjjjhhjhjff}
we have the Eikonal type equation:
\begin{equation}\label{MaxVacFullPPNmmmffffffiuiuhjuughbghhiuijghghghhjhjhhghyuyjjjhhjhjffkkl}
\frac{1}{c^2_0}\left(\frac{\partial S}{\partial t}+\vec {\tilde
u}\cdot\nabla_\vec x S\right)^2=\left|\nabla_\vec x S\right|^2,
\end{equation}
and we rewrite the equation of the propagation of the beam
\er{MaxVacFullPPNmmmffffffiuiuhjuughbghhiuijghghghhhfhhghghguygtjjjkkjjkkhh}:
\begin{multline}\label{MaxVacFullPPNmmmffffffiuiuhjuughbghhiuijghghghhhfhhghghguygtjjjkkjjkkhhklkl}
\left(\frac{\partial}{\partial
t}\left\{\frac{1}{c^2_0}\left(\frac{\partial S}{\partial t}+\vec
{\tilde u}\cdot\nabla_{\vec x}S\right)\right\}+\Div_{\vec x}
\left\{\frac{1}{c^2_0}\left(\frac{\partial S}{\partial t}+\vec
{\tilde u}\cdot\nabla_{\vec x} S\right)\vec {\tilde
u}\right\}-\left(\Delta_{\vec x}S\right)\right)A
\\+\frac{2}{c^2_0}\left(\frac{\partial S}{\partial
t}+\vec {\tilde u}\cdot\nabla_{\vec x}S\right)\left(\frac{\partial
A}{\partial t}+\vec {\tilde u}\cdot\nabla_{\vec x}A\right)
-2\nabla_{\vec x}S\cdot\nabla_{\vec x}A=0
\end{multline}
Then, as before, denoting
\begin{equation}\label{MaxVacFullPPNmmmffffffiuiuhjuughbghhiuijghghghhhgjgfghjdfjgddg11}
\vec h(\vec x,t):=\vec {\tilde u}(\vec x,t)-c^2_0(\vec
x,t)\left(\frac{\partial S}{\partial t}(\vec x,t)+\vec {\tilde
u}(\vec x,t)\cdot\nabla_{\vec x}S(\vec x,t)\right)^{-1}\nabla_{\vec
x}S(\vec x,t),
\end{equation}
and
\begin{multline}\label{MaxVacFullPPNmmmffffffiuiuhjuughbghhiuijghghghhhfhhghghguygtjuuujjjkhjhjh11kklkloiouu}
G\left(\vec x,t\right):=\\
\frac{c^2_0}{2\left(\frac{\partial S}{\partial t}+\vec {\tilde
u}\cdot\nabla_{\vec x}S\right)}\left(\Delta_{\vec
x}S-\left(\frac{\partial}{\partial
t}\left\{\frac{1}{c^2_0}\left(\frac{\partial S}{\partial t}+\vec
{\tilde u}\cdot\nabla_{\vec x}S\right)\right\}+\Div_{\vec x}
\left\{\frac{1}{c^2_0}\left(\frac{\partial S}{\partial t}+\vec
{\tilde u}\cdot\nabla_{\vec x} S\right)\vec {\tilde
u}\right\}\right)\right),
\end{multline}
by inserting
\er{MaxVacFullPPNmmmffffffiuiuhjuughbghhiuijghghghhhgjgfghjdfjgddg11}
and
\er{MaxVacFullPPNmmmffffffiuiuhjuughbghhiuijghghghhhfhhghghguygtjuuujjjkhjhjh11kklkloiouu}
into
\er{MaxVacFullPPNmmmffffffiuiuhjuughbghhiuijghghghhhfhhghghguygtjjjkkjjkkhhklkl}
we clearly have:
\begin{equation}\label{MaxVacFullPPNmmmffffffiuiuhjuughbghhiuijghghghhhfhhghghguygtjjjkkjjkkhhklklkjk}
\frac{\partial A}{\partial t}(\vec x,t)
+\vec h(\vec x,t)\cdot\nabla_{\vec x}A(\vec x,t)=
G(\vec x,t) A(\vec x,t)
.
\end{equation}
Next  consider a curve $\vec r(t):[t_0,b]\to\mathbb{R}^3$,
parameterized by the time variable $t$, defined as a solution of
ordinary differential equation:
\begin{equation}\label{MaxVacFullPPNmmmffffffiuiuhjuughbghhiuijghghghhhfhhghghguygtjuuujjjkk1133}
\begin{cases}
\frac{d\vec r}{dt}(t)=\vec h\left(\vec r(t),t\right)\\
\vec r(t_0)=\vec x_0,
\end{cases}
\end{equation}
here $\vec h$ was defined in
\er{MaxVacFullPPNmmmffffffiuiuhjuughbghhiuijghghghhhgjgfghjdfjgddg11}.
Note that the equations of the ray propagation of the form
\er{MaxVacFullPPNmmmffffffiuiuhjuughbghhiuijghghghhhfhhghghguygtjuuujjjkk1133}
are not always convenient since $\vec h$ in the right hand side of
\er{MaxVacFullPPNmmmffffffiuiuhjuughbghhiuijghghghhhfhhghghguygtjuuujjjkk11}
depend on the partial derivatives of the quantity $S$, which we do
not know apriory unless we solve
\er{MaxVacFullPPNmmmffffffiuiuhjuughbghhiuijghghghhjhjhhghyuyjjjhhjhjffkkl}.
In the following proposition we present the alternative form for the
equations of the ray propagation being the second-order ordinary
differential equations which dose not contain the quantity $S$ or
its partial derivatives.
\begin{proposition}\label{profghgjjh}
Consider a smooth solution of
\er{MaxVacFullPPNmmmffffffiuiuhjuughbghhiuijghghghhjhjhhghyuyjjjhhjhjffkkl}.
Next let $\vec r(t):[t_0,b]\to\mathbb{R}^3$ be any curve
parameterized by the time variable $t$, defined as a solution
\er{MaxVacFullPPNmmmffffffiuiuhjuughbghhiuijghghghhhfhhghghguygtjuuujjjkk1133}.
Then, $\vec r(t)$ satisfies the following second-order ordinary
differential equations of the Ray Propagation:
\begin{multline}\label{MaxVacFullPPNmmmffffffiuiuhjuughbghhiuijghghghhhfhhghghguygtjuuujjjkk11jjjggjjgtttjhhuhhjhjhgghghjhjggjugghgh}
\frac{d}{dt}\left(\frac{1}{c_0\left(\vec
r(t),t\right)}\left(\frac{d\vec r}{dt}(t)-\vec {\tilde u}\left(\vec
r(t),t\right)\right)\right)=
 -\nabla_\vec x c_0(\vec r,t)-\frac{1}{c_0(\vec r,t)}\left\{d_{\vec x}\vec {\tilde u}(\vec
r,t)\right\}^T\cdot\left(\frac{d\vec r}{dt}-\vec {\tilde
u}(\vec r,t)\right)\\
+\frac{1}{c^2_0(\vec r,t)}\left(\left(\frac{d\vec r}{dt}-\vec
{\tilde u}(\vec r,t)\right)\cdot\nabla_{\vec x}c_0(\vec
r,t)\right)\left(\frac{d\vec r}{dt}-\vec {\tilde
u}(\vec r,t)\right)\\
+\frac{1}{2c^3_0(\vec r,t)}\left(\left(\left(\left\{d_{\vec x}\vec
{\tilde u}(\vec r,t)\right\}^T+d_{\vec x}\vec {\tilde u}(\vec
r,t)\right)\cdot\left(\frac{d\vec r}{dt}-\vec {\tilde u}(\vec
r,t)\right)\right)\cdot \left(\frac{d\vec r}{dt}-\vec {\tilde
u}(\vec r,t)\right)\right)\left(\frac{d\vec r}{dt}-\vec {\tilde
u}(\vec r,t)\right).
\end{multline}
Moreover, the proper scalar quantity
$\tilde\omega(t):[t_0,b]\to\mathbb{R}$, defined by
\begin{equation}\label{MaxVacFullPPNmmmffffffiuiuhjuughbghhiuijghghghhhfhhghghguygtjuuujjjkk11jjjggjjgyhhjhhj}
\tilde\omega(t)=\frac{\partial S}{\partial t}\left(\vec
r(t),t\right)+\vec {\tilde u}\left(\vec
r(t),t\right)\cdot\nabla_{\vec x}S\left(\vec r(t),t\right)
\end{equation}
satisfies the following ordinary differential equation:
\begin{multline}\label{MaxVacFullPPNmmmffffffiuiuhjuughbghhiuijghghghhhfhhghghguygtjuuujjjkk11jjjggjjgkhhhhkjhkhkhhkkkhkkhhhhjggghjhj}
\frac{d\tilde\omega}{dt}(t)= \frac{1}{c_0(\vec
r,t)}\left(\frac{\partial c_0}{\partial t}(\vec r,t)+\vec {\tilde
u}(\vec r,t)\cdot\nabla_{\vec x}c_0(\vec
r,t)\right)\tilde\omega(t)\\-\frac{1}{2c^2_0(\vec
r,t)}\left(\left(\left(\left\{d_{\vec x}\vec {\tilde u}(\vec
r,t)\right\}^T+d_{\vec x}\vec {\tilde u}(\vec
r,t)\right)\cdot\left(\frac{d\vec r}{dt}(t)-\vec {\tilde u}(\vec
r,t)\right)\right)\cdot \left(\frac{d\vec r}{dt}(t)-\vec {\tilde
u}(\vec r,t)\right)\right)\tilde\omega(t).
\end{multline}
\end{proposition}
\begin{proof}
Consider $\vec r(t):[t_0,b]\to\mathbb{R}^3$ to be an arbitrary
solution of ordinary differential equation
\er{MaxVacFullPPNmmmffffffiuiuhjuughbghhiuijghghghhhfhhghghguygtjuuujjjkk11},
i.e.:
\begin{equation}\label{MaxVacFullPPNmmmffffffiuiuhjuughbghhiuijghghghhhfhhghghguygtjuuujjjkk11jjj}
\begin{cases}
\frac{d\vec r}{dt}(t)=\vec h\left(\vec r(t),t\right)\\
\vec r(t_0)=\vec x_0,
\end{cases}
\end{equation}
where $\vec h$ was defined in
\er{MaxVacFullPPNmmmffffffiuiuhjuughbghhiuijghghghhhgjgfghjdfjgddg11}.
In particular, by
\er{MaxVacFullPPNmmmffffffiuiuhjuughbghhiuijghghghhhfhhghghguygtjuuujjjkk11jjj}
and
\er{MaxVacFullPPNmmmffffffiuiuhjuughbghhiuijghghghhjhjhhghyuyjjjhhjhjffkklkljljjkhhjkjhghghghl;kl;;k;kkljjkjkll;l;}
we have
\begin{equation}\label{MaxVacFullPPNmmmffffffiuiuhjuughbghhiuijghghghhhfhhghghguygtjuuujjjkk11jjjggjjgjkljkuhu}
\left|\frac{d\vec r}{dt}(t)-\vec {\tilde u}(\vec
r,t)\right|^2=c^2_0(\vec r,t),
\end{equation}
and by
\er{MaxVacFullPPNmmmffffffiuiuhjuughbghhiuijghghghhhfhhghghguygtjuuujjjkk11jjj}
and
\er{MaxVacFullPPNmmmffffffiuiuhjuughbghhiuijghghghhhgjgfghjdfjgddg11}
we have
\begin{equation}\label{MaxVacFullPPNmmmffffffiuiuhjuughbghhiuijghghghhhfhhghghguygtjuuujjjkk11jjjggjjgyhhj}
\nabla_{\vec x}S(\vec r,t)=-\frac{1}{c^2_0(\vec
r,t)}\left(\frac{\partial S}{\partial t}(\vec r,t)+\vec {\tilde
u}(\vec r,t)\cdot\nabla_{\vec x}S(\vec r,t)\right)\left(\frac{d\vec
r}{dt}(t)-\vec {\tilde u}(\vec r,t)\right).
\end{equation}
Then, by
\er{MaxVacFullPPNmmmffffffiuiuhjuughbghhiuijghghghhhfhhghghguygtjuuujjjkk11jjjggjjgyhhj}
and
\er{MaxVacFullPPNmmmffffffiuiuhjuughbghhiuijghghghhhfhhghghguygtjuuujjjkk11jjjggjjgjkljkuhu}
obviously we have
\begin{equation}\label{MaxVacFullPPNmmmffffffiuiuhjuughbghhiuijghghghhhfhhghghguygtjuuujjjkk11jjjggjjg}
\frac{d\vec r}{dt}(t)-\vec {\tilde u}(\vec r,t)=-c^2_0(\vec
r,t)\left(\frac{\partial S}{\partial t}(\vec r,t)+\vec {\tilde
u}(\vec r,t)\cdot\nabla_{\vec x}S(\vec r,t)\right)^{-1}\nabla_{\vec
x}S(\vec r,t),
\end{equation}
\begin{equation}\label{MaxVacFullPPNmmmffffffiuiuhjuughbghhiuijghghghhhfhhghghguygtjuuujjjkk11jjjggjjgpkkkjkjkjk}
\frac{\partial S}{\partial t}(\vec r,t)+\vec {\tilde u}(\vec
r,t)\cdot\nabla_{\vec x}S(\vec r,t)=-\left(\frac{d\vec r}{dt}-\vec
{\tilde u}(\vec r,t)\right)\cdot\nabla_{\vec x}S(\vec r,t),
\end{equation}
and thus,
\begin{equation}\label{MaxVacFullPPNmmmffffffiuiuhjuughbghhiuijghghghhhfhhghghguygtjuuujjjkk11jjjggjjg22}
\frac{d\vec r}{dt}(t)-\vec {\tilde u}(\vec r,t)=c^2_0(\vec
r,t)\left(\left(\frac{d\vec r}{dt}-\vec {\tilde u}(\vec
r,t)\right)\cdot\nabla_{\vec x}S(\vec r,t)\right)^{-1}\nabla_{\vec
x}S(\vec r,t).
\end{equation}
On the other hand by
\er{MaxVacFullPPNmmmffffffiuiuhjuughbghhiuijghghghhjhjhhghyuyjjjhhjhjffkkl}
and the Chain Rule we have
\begin{equation}\label{MaxVacFullPPNmmmffffffiuiuhjuughbghhiuijghghghhjhjhhghyuyjjjhhjhjffkkljjjmn,n33}
\frac{1}{c^2_0}\left(\frac{\partial S}{\partial t}+\vec {\tilde
u}\cdot\nabla_\vec x S\right)\nabla_{\vec x}\left(\frac{\partial
S}{\partial t}+\vec {\tilde u}\cdot\nabla_\vec x
S\right)-\frac{1}{c^3_0}\left(\frac{\partial S}{\partial t}+\vec
{\tilde u}\cdot\nabla_\vec x S\right)^2\nabla_\vec x
c_0=\nabla^2_{\vec x}S(\vec r,t)\cdot\nabla_{\vec x}S(\vec r,t).
\end{equation}
Thus,
\begin{equation}\label{MaxVacFullPPNmmmffffffiuiuhjuughbghhiuijghghghhjhjhhghyuyjjjhhjhjffkkljjjmn,n44}
\left(\frac{\partial S}{\partial t}+\vec {\tilde u}\cdot\nabla_\vec
x S\right)\nabla_{\vec x}\left(\frac{\partial S}{\partial t}+\vec
{\tilde u}\cdot\nabla_\vec x
S\right)-\frac{1}{c_0}\left(\frac{\partial S}{\partial t}+\vec
{\tilde u}\cdot\nabla_\vec x S\right)^2\nabla_\vec x
c_0=c^2_0\nabla^2_{\vec x}S(\vec r,t)\cdot\nabla_{\vec x}S(\vec
r,t).
\end{equation}
So,
\begin{multline}\label{MaxVacFullPPNmmmffffffiuiuhjuughbghhiuijghghghhjhjhhghyuyjjjhhjhjffkkljjjmn,n77}
\frac{\partial}{\partial t}\left(\nabla_{\vec
x}S\right)+\nabla^2_{\vec x}S\cdot\vec {\tilde u}+\left\{d_{\vec
x}\vec {\tilde u}\right\}^T\cdot\nabla_{\vec
x}S-\frac{1}{c_0}\left(\frac{\partial S}{\partial t}+\vec {\tilde
u}\cdot\nabla_\vec x S\right)\nabla_\vec x
c_0\\=c^2_0\left(\frac{\partial S}{\partial t}+\vec {\tilde
u}\cdot\nabla_\vec x S\right)^{-1}\nabla^2_{\vec x}S(\vec
r,t)\cdot\nabla_{\vec x}S(\vec r,t).
\end{multline}
Therefore,
\begin{multline}\label{MaxVacFullPPNmmmffffffiuiuhjuughbghhiuijghghghhjhjhhghyuyjjjhhjhjffkkljjjmn,n}
\frac{\partial}{\partial t}\left(\nabla_{\vec
x}S\right)+\nabla^2_{\vec x}S\cdot\vec {\tilde u}-
c^2_0\left(\frac{\partial S}{\partial t}+\vec {\tilde
u}\cdot\nabla_\vec x S\right)^{-1}\nabla^2_{\vec x}S(\vec
r,t)\cdot\nabla_{\vec x}S(\vec r,t)=\\
\frac{1}{c_0}\left(\frac{\partial S}{\partial t}+\vec {\tilde
u}\cdot\nabla_\vec x S\right)\nabla_\vec x c_0 -\left\{d_{\vec
x}\vec {\tilde u}\right\}^T\cdot\nabla_{\vec x}S.
\end{multline}
Next, by
\er{MaxVacFullPPNmmmffffffiuiuhjuughbghhiuijghghghhjhjhhghyuyjjjhhjhjffkkljjjmn,n}
and
\er{MaxVacFullPPNmmmffffffiuiuhjuughbghhiuijghghghhhfhhghghguygtjuuujjjkk11jjjggjjg}
we infer
\begin{multline}\label{MaxVacFullPPNmmmffffffiuiuhjuughbghhiuijghghghhhfhhghghguygtjuuujjjkk11jjjggjjgijkjuu}
\nabla^2_{\vec x}S(\vec r,t)\cdot\left(\frac{d\vec r}{dt}-\vec
{\tilde u}(\vec r,t)\right)+\frac{\partial}{\partial
t}\left(\nabla_{\vec x}S\right)(\vec r,t)+\nabla^2_{\vec x}S(\vec
r,t)\cdot\vec {\tilde u}(\vec r,t)=\\ \frac{1}{c_0(\vec
r,t)}\left(\frac{\partial S}{\partial t}(\vec r,t)+\vec {\tilde
u}(\vec r,t)\cdot\nabla_\vec x S(\vec r,t)\right)\nabla_\vec x
c_0(\vec r,t) -\left\{d_{\vec x}\vec {\tilde u}(\vec
r,t)\right\}^T\cdot\nabla_{\vec x}S(\vec r,t).
\end{multline}
On the other hand, by
\er{MaxVacFullPPNmmmffffffiuiuhjuughbghhiuijghghghhhfhhghghguygtjuuujjjkk11jjjggjjgpkkkjkjkjk}
and and the Chain Rule we have
\begin{multline}\label{MaxVacFullPPNmmmffffffiuiuhjuughbghhiuijghghghhhfhhghghguygtjuuujjjkk11jjjggjjgkhhhh}
\frac{d}{dt}\left(\frac{\partial S}{\partial t}\left(\vec
r(t),t\right)+\vec {\tilde u}\left(\vec
r(t),t\right)\cdot\nabla_{\vec x}S\left(\vec
r(t),t\right)\right)=-\frac{d}{dt}\left(\left(\frac{d\vec
r}{dt}(t)-\vec {\tilde u}\left(\vec
r(t),t\right)\right)\cdot\nabla_{\vec x}S\left(\vec
r(t),t\right)\right)=\\
-\frac{d}{dt}\left(\frac{d\vec r}{dt}(t)-\vec {\tilde u}\left(\vec
r(t),t\right)\right)\cdot\nabla_{\vec x}S(\vec r,t)
-\left(\frac{d\vec r}{dt}-\vec {\tilde u}(\vec
r,t)\right)\cdot\left(\nabla^2_{\vec x}S(\vec r,t)\cdot\frac{d\vec
r}{dt}+\frac{\partial}{\partial t}\left(\nabla_{\vec x}S\right)(\vec
r,t)\right)
\\=
-\frac{d}{dt}\left(\frac{d\vec r}{dt}(t)-\vec {\tilde u}\left(\vec
r(t),t\right)\right)\cdot\nabla_{\vec x}S(\vec r,t)
\\-\left(\frac{d\vec r}{dt}-\vec {\tilde u}(\vec
r,t)\right)\cdot\left(\nabla^2_{\vec x}S(\vec
r,t)\cdot\left(\frac{d\vec r}{dt}-\vec {\tilde u}(\vec
r,t)\right)+\frac{\partial}{\partial t}\left(\nabla_{\vec
x}S\right)(\vec r,t)+\nabla^2_{\vec x}S(\vec r,t)\cdot\vec {\tilde
u}(\vec r,t)\right).
\end{multline}
Thus, inserting
\er{MaxVacFullPPNmmmffffffiuiuhjuughbghhiuijghghghhhfhhghghguygtjuuujjjkk11jjjggjjgyhhj}
and
\er{MaxVacFullPPNmmmffffffiuiuhjuughbghhiuijghghghhhfhhghghguygtjuuujjjkk11jjjggjjgijkjuu}
into
\er{MaxVacFullPPNmmmffffffiuiuhjuughbghhiuijghghghhhfhhghghguygtjuuujjjkk11jjjggjjgkhhhh}
we deduce
\begin{multline}\label{MaxVacFullPPNmmmffffffiuiuhjuughbghhiuijghghghhhfhhghghguygtjuuujjjkk11jjjggjjgkhhhhkjhkh}
\frac{d}{dt}\left(\frac{\partial S}{\partial t}\left(\vec
r(t),t\right)+\vec {\tilde u}\left(\vec
r(t),t\right)\cdot\nabla_{\vec x}S\left(\vec r(t),t\right)\right)
\\=
\frac{1}{c^2_0(\vec r,t)}\left(\frac{\partial S}{\partial t}(\vec
r,t)+\vec {\tilde u}(\vec r,t)\cdot\nabla_{\vec x}S(\vec
r,t)\right)\left(\frac{d\vec r}{dt}-\vec {\tilde u}(\vec
r,t)\right)\cdot\frac{d}{dt}\left(\frac{d\vec r}{dt}(t)-\vec {\tilde
u}\left(\vec r(t),t\right)\right) \\-\left(\frac{d\vec r}{dt}-\vec
{\tilde u}(\vec r,t)\right)\cdot\left(\frac{1}{c_0(\vec
r,t)}\left(\frac{\partial S}{\partial t}(\vec r,t)+\vec {\tilde
u}(\vec r,t)\cdot\nabla_\vec x S(\vec r,t)\right)\nabla_\vec x
c_0(\vec r,t) -\left\{d_{\vec x}\vec {\tilde u}(\vec
r,t)\right\}^T\cdot\nabla_{\vec x}S(\vec r,t)\right)
\\=
\frac{1}{2c^2_0(\vec r,t)}\left(\frac{\partial S}{\partial t}(\vec
r,t)+\vec {\tilde u}(\vec r,t)\cdot\nabla_{\vec x}S(\vec
r,t)\right)\frac{d}{dt}\left(\left|\frac{d\vec r}{dt}(t)-\vec
{\tilde u}\left(\vec r(t),t\right)\right|^2\right)
\\-\left(\frac{d\vec r}{dt}-\vec {\tilde u}(\vec
r,t)\right)\cdot\left(\frac{1}{c_0(\vec r,t)}\left(\frac{\partial
S}{\partial t}(\vec r,t)+\vec {\tilde u}(\vec r,t)\cdot\nabla_\vec x
S(\vec r,t)\right)\nabla_\vec x c_0(\vec r,t) -\left\{d_{\vec x}\vec
{\tilde u}(\vec r,t)\right\}^T\cdot\nabla_{\vec x}S(\vec
r,t)\right).
\end{multline}
Thus, inserting
\er{MaxVacFullPPNmmmffffffiuiuhjuughbghhiuijghghghhhfhhghghguygtjuuujjjkk11jjjggjjgjkljkuhu}
into
\er{MaxVacFullPPNmmmffffffiuiuhjuughbghhiuijghghghhhfhhghghguygtjuuujjjkk11jjjggjjgkhhhhkjhkh}
we deduce
\begin{multline}\label{MaxVacFullPPNmmmffffffiuiuhjuughbghhiuijghghghhhfhhghghguygtjuuujjjkk11jjjggjjgkhhhhkjhkhkhhkkkhk}
\frac{d}{dt}\left(\frac{\partial S}{\partial t}\left(\vec
r(t),t\right)+\vec {\tilde u}\left(\vec
r(t),t\right)\cdot\nabla_{\vec x}S\left(\vec r(t),t\right)\right)
\\=
\frac{1}{2c^2_0(\vec r,t)}\left(\frac{\partial S}{\partial t}(\vec
r,t)+\vec {\tilde u}(\vec r,t)\cdot\nabla_{\vec x}S(\vec
r,t)\right)\frac{d}{dt}\left(c^2_0\left(\vec r(t),t\right)\right)
\\-\left(\frac{d\vec r}{dt}-\vec {\tilde u}(\vec
r,t)\right)\cdot\left(\frac{1}{c_0(\vec r,t)}\left(\frac{\partial
S}{\partial t}(\vec r,t)+\vec {\tilde u}(\vec r,t)\cdot\nabla_\vec x
S(\vec r,t)\right)\nabla_\vec x c_0(\vec r,t) -\left\{d_{\vec x}\vec
{\tilde u}(\vec r,t)\right\}^T\cdot\nabla_{\vec x}S(\vec
r,t)\right)=
\\
\frac{1}{c_0(\vec r,t)}\left(\frac{\partial S}{\partial t}(\vec
r,t)+\vec {\tilde u}(\vec r,t)\cdot\nabla_{\vec x}S(\vec
r,t)\right)\left(\nabla_{\vec x}c_0(\vec r,t)\cdot\left(\frac{d\vec
r}{dt}-\vec {\tilde u}(\vec r,t)\right)+\frac{\partial c_0}{\partial
t}(\vec r,t)+\vec {\tilde u}(\vec r,t)\cdot\nabla_{\vec x}c_0(\vec
r,t)\right)
\\-\left(\frac{d\vec r}{dt}-\vec {\tilde u}(\vec
r,t)\right)\cdot\left(\frac{1}{c_0(\vec r,t)}\left(\frac{\partial
S}{\partial t}(\vec r,t)+\vec {\tilde u}(\vec r,t)\cdot\nabla_\vec x
S(\vec r,t)\right)\nabla_\vec x c_0(\vec r,t) -\left\{d_{\vec x}\vec
{\tilde u}(\vec r,t)\right\}^T\cdot\nabla_{\vec x}S(\vec
r,t)\right).
\end{multline}
Then inserting
\er{MaxVacFullPPNmmmffffffiuiuhjuughbghhiuijghghghhhfhhghghguygtjuuujjjkk11jjjggjjgyhhj}
into
\er{MaxVacFullPPNmmmffffffiuiuhjuughbghhiuijghghghhhfhhghghguygtjuuujjjkk11jjjggjjgkhhhhkjhkhkhhkkkhk}
we infer
\begin{multline}\label{MaxVacFullPPNmmmffffffiuiuhjuughbghhiuijghghghhhfhhghghguygtjuuujjjkk11jjjggjjgkhhhhkjhkhkhhkkkhkkhhhh}
\frac{d}{dt}\left(\frac{\partial S}{\partial t}\left(\vec
r(t),t\right)+\vec {\tilde u}\left(\vec
r(t),t\right)\cdot\nabla_{\vec x}S\left(\vec r(t),t\right)\right)=
\\
\frac{1}{c_0(\vec r,t)}\left(\frac{\partial S}{\partial t}(\vec
r,t)+\vec {\tilde u}(\vec r,t)\cdot\nabla_{\vec x}S(\vec
r,t)\right)\left(\nabla_{\vec x}c_0(\vec r,t)\cdot\left(\frac{d\vec
r}{dt}-\vec {\tilde u}(\vec r,t)\right)+\frac{\partial c_0}{\partial
t}(\vec r,t)+\vec {\tilde u}(\vec r,t)\cdot\nabla_{\vec x}c_0(\vec
r,t)\right)
\\-\frac{1}{c^2_0(\vec
r,t)}\left(\frac{\partial S}{\partial t}(\vec r,t)+\vec {\tilde
u}(\vec r,t)\cdot\nabla_\vec x S(\vec r,t)\right)\times\\
\times\left(\frac{d\vec r}{dt}-\vec {\tilde u}(\vec
r,t)\right)\cdot\left(c_0(\vec r,t)\nabla_\vec x c_0(\vec r,t)
+\left\{d_{\vec x}\vec {\tilde u}(\vec
r,t)\right\}^T\cdot\left(\frac{d\vec r}{dt}(t)-\vec {\tilde u}(\vec
r,t)\right)\right).
\end{multline}
So
\begin{multline}\label{MaxVacFullPPNmmmffffffiuiuhjuughbghhiuijghghghhhfhhghghguygtjuuujjjkk11jjjggjjgkhhhhkjhkhkhhkkkhkkhhhhjggg}
\frac{d}{dt}\left(\frac{\partial S}{\partial t}\left(\vec
r(t),t\right)+\vec {\tilde u}\left(\vec
r(t),t\right)\cdot\nabla_{\vec x}S\left(\vec r(t),t\right)\right)=
\\
\frac{1}{c_0(\vec r,t)}\left(\frac{\partial S}{\partial t}(\vec
r,t)+\vec {\tilde u}(\vec r,t)\cdot\nabla_{\vec x}S(\vec
r,t)\right)\left(\frac{\partial c_0}{\partial t}(\vec r,t)+\vec
{\tilde u}(\vec r,t)\cdot\nabla_{\vec x}c_0(\vec
r,t)\right)\\-\frac{1}{c^2_0(\vec r,t)}\left(\frac{\partial
S}{\partial t}(\vec r,t)+\vec {\tilde u}(\vec r,t)\cdot\nabla_{\vec
x}S(\vec r,t)\right)\left(\left(\left\{d_{\vec x}\vec {\tilde
u}(\vec r,t)\right\}^T\cdot\left(\frac{d\vec r}{dt}-\vec {\tilde
u}(\vec r,t)\right)\right)\cdot \left(\frac{d\vec r}{dt}-\vec
{\tilde u}(\vec r,t)\right)\right)
\\=
\frac{1}{c_0(\vec r,t)}\left(\frac{\partial S}{\partial t}(\vec
r,t)+\vec {\tilde u}(\vec r,t)\cdot\nabla_{\vec x}S(\vec
r,t)\right)\left(\frac{\partial c_0}{\partial t}(\vec r,t)+\vec
{\tilde u}(\vec r,t)\cdot\nabla_{\vec x}c_0(\vec
r,t)\right)\\-\frac{1}{2c^2_0(\vec r,t)}\left(\frac{\partial
S}{\partial t}(\vec r,t)+\vec {\tilde u}(\vec r,t)\cdot\nabla_{\vec
x}S(\vec r,t)\right)\times\\ \times\left(\left(\left(\left\{d_{\vec
x}\vec {\tilde u}(\vec r,t)\right\}^T+d_{\vec x}\vec {\tilde u}(\vec
r,t)\right)\cdot\left(\frac{d\vec r}{dt}-\vec {\tilde u}(\vec
r,t)\right)\right)\cdot \left(\frac{d\vec r}{dt}-\vec {\tilde
u}(\vec r,t)\right)\right),
\end{multline}
So we deduce
\er{MaxVacFullPPNmmmffffffiuiuhjuughbghhiuijghghghhhfhhghghguygtjuuujjjkk11jjjggjjgkhhhhkjhkhkhhkkkhkkhhhhjggghjhj}.
Moreover, we can rewrite it as:
\begin{multline}\label{MaxVacFullPPNmmmffffffiuiuhjuughbghhiuijghghghhhfhhghghguygtjuuujjjkk11jjjggjjgkhhhhkjhkhkhhkkkhkkhhhhjgggyuyy}
\left(\frac{\partial S}{\partial t}(\vec r,t)+\vec {\tilde u}(\vec
r,t)\cdot\nabla_{\vec x}S(\vec
r,t)\right)^{-1}\frac{d}{dt}\left(\frac{\partial S}{\partial
t}\left(\vec r(t),t\right)+\vec {\tilde u}\left(\vec
r(t),t\right)\cdot\nabla_{\vec x}S\left(\vec r(t),t\right)\right)
\\=
\frac{1}{c_0(\vec r,t)}\left(\frac{\partial c_0}{\partial t}(\vec
r,t)+\vec {\tilde u}(\vec r,t)\cdot\nabla_{\vec x}c_0(\vec
r,t)\right)\\-\frac{1}{2c^2_0(\vec
r,t)}\left(\left(\left(\left\{d_{\vec x}\vec {\tilde u}(\vec
r,t)\right\}^T+d_{\vec x}\vec {\tilde u}(\vec
r,t)\right)\cdot\left(\frac{d\vec r}{dt}-\vec {\tilde u}(\vec
r,t)\right)\right)\cdot \left(\frac{d\vec r}{dt}-\vec {\tilde
u}(\vec r,t)\right)\right).
\end{multline}
On the other hand, by
\er{MaxVacFullPPNmmmffffffiuiuhjuughbghhiuijghghghhhfhhghghguygtjuuujjjkk11jjjggjjg}
we have:
\begin{multline}\label{MaxVacFullPPNmmmffffffiuiuhjuughbghhiuijghghghhhfhhghghguygtjuuujjjkk11jjjggjjgtttjhhu}
\frac{d}{dt}\left(\frac{1}{c^2_0\left(\vec
r(t),t\right)}\left(\frac{d\vec r}{dt}(t)-\vec {\tilde u}\left(\vec
r(t),t\right)\right)\right)=\\-\left(\frac{\partial S}{\partial
t}(\vec r,t)+\vec {\tilde u}(\vec r,t)\cdot\nabla_{\vec x}S(\vec
r,t)\right)^{-1}\left(\nabla^2_{\vec x}S(\vec r,t)\cdot\frac{d\vec
r}{dt}+\frac{\partial}{\partial t}\left(\nabla_{\vec x}S\right)(\vec r,t)\right)\\
+\left(\frac{\partial S}{\partial t}(\vec r,t)+\vec {\tilde u}(\vec
r,t)\cdot\nabla_{\vec x}S(\vec
r,t)\right)^{-2}\left(\frac{d}{dt}\left(\frac{\partial S}{\partial
t}\left(\vec r(t),t\right)+\vec {\tilde u}\left(\vec
r(t),t\right)\cdot\nabla_{\vec x}S\left(\vec
r(t),t\right)\right)\right)\nabla_{\vec x}S(\vec r,t)\\=
-\left(\frac{\partial S}{\partial t}(\vec r,t)+\vec {\tilde u}(\vec
r,t)\cdot\nabla_{\vec x}S(\vec r,t)\right)^{-1}\times\\
\times\left(\nabla^2_{\vec x}S(\vec r,t)\cdot\left(\frac{d\vec
r}{dt}-\vec {\tilde u}(\vec r,t)\right)+\frac{\partial}{\partial
t}\left(\nabla_{\vec x}S\right)(\vec r,t)+\nabla^2_{\vec x}S(\vec
r,t)\cdot\vec {\tilde u}(\vec
r,t)\right)\\
+\left(\frac{\partial S}{\partial t}(\vec r,t)+\vec {\tilde u}(\vec
r,t)\cdot\nabla_{\vec x}S(\vec
r,t)\right)^{-2}\left(\frac{d}{dt}\left(\frac{\partial S}{\partial
t}\left(\vec r(t),t\right)+\vec {\tilde u}\left(\vec
r(t),t\right)\cdot\nabla_{\vec x}S\left(\vec
r(t),t\right)\right)\right)\nabla_{\vec x}S(\vec r,t).
\end{multline}
Thus, by inserting
\er{MaxVacFullPPNmmmffffffiuiuhjuughbghhiuijghghghhhfhhghghguygtjuuujjjkk11jjjggjjgijkjuu}
into
\er{MaxVacFullPPNmmmffffffiuiuhjuughbghhiuijghghghhhfhhghghguygtjuuujjjkk11jjjggjjgtttjhhu}
we deduce
\begin{multline}\label{MaxVacFullPPNmmmffffffiuiuhjuughbghhiuijghghghhhfhhghghguygtjuuujjjkk11jjjggjjgtttjhhuhhjhj}
\frac{d}{dt}\left(\frac{1}{c^2_0\left(\vec
r(t),t\right)}\left(\frac{d\vec r}{dt}(t)-\vec {\tilde u}\left(\vec
r(t),t\right)\right)\right)= -\left(\frac{\partial S}{\partial
t}(\vec r,t)+\vec {\tilde u}(\vec
r,t)\cdot\nabla_{\vec x}S(\vec r,t)\right)^{-1}\times\\
\times\left(\frac{1}{c_0(\vec r,t)}\left(\frac{\partial S}{\partial
t}(\vec r,t)+\vec {\tilde u}(\vec r,t)\cdot\nabla_\vec x S(\vec
r,t)\right)\nabla_\vec x c_0(\vec r,t) -\left\{d_{\vec x}\vec
{\tilde u}(\vec
r,t)\right\}^T\cdot\nabla_{\vec x}S(\vec r,t)\right)\\
+\left(\frac{\partial S}{\partial t}(\vec r,t)+\vec {\tilde u}(\vec
r,t)\cdot\nabla_{\vec x}S(\vec
r,t)\right)^{-2}\left(\frac{d}{dt}\left(\frac{\partial S}{\partial
t}\left(\vec r(t),t\right)+\vec {\tilde u}\left(\vec
r(t),t\right)\cdot\nabla_{\vec x}S\left(\vec
r(t),t\right)\right)\right)\nabla_{\vec x}S(\vec r,t)\\=
 -\left(\frac{1}{c_0(\vec r,t)}\nabla_\vec x c_0(\vec r,t) -\left(\frac{\partial S}{\partial t}(\vec r,t)+\vec {\tilde
u}(\vec r,t)\cdot\nabla_{\vec x}S(\vec
r,t)\right)^{-1}\left\{d_{\vec x}\vec {\tilde u}(\vec
r,t)\right\}^T\cdot\nabla_{\vec x}S(\vec r,t)\right)\\
+\left(\frac{\partial S}{\partial t}(\vec r,t)+\vec {\tilde u}(\vec
r,t)\cdot\nabla_{\vec x}S(\vec
r,t)\right)^{-2}\left(\frac{d}{dt}\left(\frac{\partial S}{\partial
t}\left(\vec r(t),t\right)+\vec {\tilde u}\left(\vec
r(t),t\right)\cdot\nabla_{\vec x}S\left(\vec
r(t),t\right)\right)\right)\nabla_{\vec x}S(\vec r,t).
\end{multline}
Thus inserting
\er{MaxVacFullPPNmmmffffffiuiuhjuughbghhiuijghghghhhfhhghghguygtjuuujjjkk11jjjggjjgkhhhhkjhkhkhhkkkhkkhhhhjgggyuyy}
into
\er{MaxVacFullPPNmmmffffffiuiuhjuughbghhiuijghghghhhfhhghghguygtjuuujjjkk11jjjggjjgtttjhhuhhjhj}
gives
\begin{multline}\label{MaxVacFullPPNmmmffffffiuiuhjuughbghhiuijghghghhhfhhghghguygtjuuujjjkk11jjjggjjgtttjhhuhhjhjhgghg}
\frac{d}{dt}\left(\frac{1}{c^2_0\left(\vec
r(t),t\right)}\left(\frac{d\vec r}{dt}(t)-\vec {\tilde u}\left(\vec
r(t),t\right)\right)\right)=\\
 -\frac{1}{c_0(\vec r,t)}\nabla_\vec x c_0(\vec r,t) +\left(\frac{\partial S}{\partial t}(\vec r,t)+\vec {\tilde
u}(\vec r,t)\cdot\nabla_{\vec x}S(\vec
r,t)\right)^{-1}\left\{d_{\vec x}\vec {\tilde u}(\vec
r,t)\right\}^T\cdot\nabla_{\vec x}S(\vec r,t)\\
+\frac{1}{c_0(\vec r,t)}\left(\frac{\partial c_0}{\partial t}(\vec
r,t)+\vec {\tilde u}(\vec r,t)\cdot\nabla_{\vec x}c_0(\vec
r,t)\right)\left(\frac{\partial S}{\partial t}(\vec r,t)+\vec
{\tilde u}(\vec r,t)\cdot\nabla_{\vec x}S(\vec
r,t)\right)^{-1}\nabla_{\vec x}S(\vec
r,t)\\
-\frac{1}{2c^2_0(\vec r,t)}\left(\left(\left(\left\{d_{\vec x}\vec
{\tilde u}(\vec r,t)\right\}^T+d_{\vec x}\vec {\tilde u}(\vec
r,t)\right)\cdot\left(\frac{d\vec r}{dt}-\vec {\tilde u}(\vec
r,t)\right)\right)\cdot \left(\frac{d\vec r}{dt}-\vec {\tilde
u}(\vec r,t)\right)\right)\times\\ \times\left(\frac{\partial
S}{\partial t}(\vec r,t)+\vec {\tilde u}(\vec r,t)\cdot\nabla_{\vec
x}S(\vec r,t)\right)^{-1}\nabla_{\vec x}S(\vec r,t).
\end{multline}
Therefore, inserting
\er{MaxVacFullPPNmmmffffffiuiuhjuughbghhiuijghghghhhfhhghghguygtjuuujjjkk11jjjggjjg}
into
\er{MaxVacFullPPNmmmffffffiuiuhjuughbghhiuijghghghhhfhhghghguygtjuuujjjkk11jjjggjjgtttjhhuhhjhjhgghg}
we finally deduce
\begin{multline}\label{MaxVacFullPPNmmmffffffiuiuhjuughbghhiuijghghghhhfhhghghguygtjuuujjjkk11jjjggjjgtttjhhuhhjhjhgghghjhjggj}
\frac{d}{dt}\left(\frac{1}{c^2_0\left(\vec
r(t),t\right)}\left(\frac{d\vec r}{dt}(t)-\vec {\tilde u}\left(\vec
r(t),t\right)\right)\right)=
 -\frac{1}{c_0(\vec r,t)}\nabla_\vec x c_0(\vec r,t)-\frac{1}{c^2_0(\vec r,t)}\left\{d_{\vec x}\vec {\tilde u}(\vec
r,t)\right\}^T\cdot\left(\frac{d\vec r}{dt}-\vec {\tilde
u}(\vec r,t)\right)\\
-\frac{1}{c^3_0(\vec r,t)}\left(\frac{\partial c_0}{\partial t}(\vec
r,t)+\vec {\tilde u}(\vec r,t)\cdot\nabla_{\vec x}c_0(\vec
r,t)\right)\left(\frac{d\vec r}{dt}-\vec {\tilde
u}(\vec r,t)\right)\\
+\frac{1}{2c^4_0(\vec r,t)}\left(\left(\left(\left\{d_{\vec x}\vec
{\tilde u}(\vec r,t)\right\}^T+d_{\vec x}\vec {\tilde u}(\vec
r,t)\right)\cdot\left(\frac{d\vec r}{dt}-\vec {\tilde u}(\vec
r,t)\right)\right)\cdot \left(\frac{d\vec r}{dt}-\vec {\tilde
u}(\vec r,t)\right)\right)\left(\frac{d\vec r}{dt}-\vec {\tilde
u}(\vec r,t)\right),
\end{multline}
that we can finally rewrite in the alternative form of
\er{MaxVacFullPPNmmmffffffiuiuhjuughbghhiuijghghghhhfhhghghguygtjuuujjjkk11jjjggjjgtttjhhuhhjhjhgghghjhjggjugghgh}.
\end{proof}
Note that, by
\er{MaxVacFullPPNmmmffffffiuiuhjuughbghhiuijghghghhhfhhghghguygtjuuujjjkk11jjjggjjgtttjhhuhhjhjhgghghjhjggjugghgh}
we obviously deduce
\begin{multline}\label{MaxVacFullPPNmmmffffffiuiuhjuughbghhiuijghghghhhfhhghghguygtjuuujjjkk11jjjggjjgtttjhhuhhjhjhgghghjhjggjugghghgjjgjg}
\frac{1}{2}\frac{d}{dt}\left(\left|\frac{1}{c_0\left(\vec
r(t),t\right)}\left(\frac{d\vec r}{dt}(t)-\vec {\tilde u}\left(\vec
r(t),t\right)\right)\right|^2-1\right)=\\
\frac{1}{c_0(\vec r,t)}\left(\left(\frac{d\vec r}{dt}-\vec {\tilde
u}(\vec r,t)\right)\cdot\nabla_{\vec x}c_0(\vec
r,t)\right)\left(\left|\frac{1}{c_0\left(\vec
r(t),t\right)}\left(\frac{d\vec r}{dt}(t)-\vec {\tilde u}\left(\vec
r(t),t\right)\right)\right|^2-1\right)\\
+\frac{1}{2c^2_0(\vec r,t)}\left(\left(\left(\left\{d_{\vec x}\vec
{\tilde u}(\vec r,t)\right\}^T+d_{\vec x}\vec {\tilde u}(\vec
r,t)\right)\cdot\left(\frac{d\vec r}{dt}-\vec {\tilde u}(\vec
r,t)\right)\right)\cdot \left(\frac{d\vec r}{dt}-\vec {\tilde
u}(\vec r,t)\right)\right)\times\\
\times\left(\left|\frac{1}{c_0\left(\vec
r(t),t\right)}\left(\frac{d\vec r}{dt}(t)-\vec {\tilde u}\left(\vec
r(t),t\right)\right)\right|^2-1\right).
\end{multline}
Therefore, if $\vec r(t)$ satisfies
\er{MaxVacFullPPNmmmffffffiuiuhjuughbghhiuijghghghhhfhhghghguygtjuuujjjkk11jjjggjjgtttjhhuhhjhjhgghghjhjggjugghgh}
then by
\er{MaxVacFullPPNmmmffffffiuiuhjuughbghhiuijghghghhhfhhghghguygtjuuujjjkk11jjjggjjgtttjhhuhhjhjhgghghjhjggjugghghgjjgjg}
we obtain that the following equality for the initial instant of
time $t_0$:
\begin{equation*}
\left|\frac{1}{c_0\left(\vec r(t_0),t_0\right)}\left(\frac{d\vec
r}{dt}(t_0)-\vec {\tilde u}\left(\vec
r(t_0),t_0\right)\right)\right|^2-1=0,
\end{equation*}
implies
\begin{equation*}
\left|\frac{1}{c_0\left(\vec r(t),t\right)}\left(\frac{d\vec
r}{dt}(t)-\vec {\tilde u}\left(\vec
r(t),t\right)\right)\right|^2-1=0
\end{equation*}
for every instant of time. So we get
\er{MaxVacFullPPNmmmffffffiuiuhjuughbghhiuijghghghhhfhhghghguygtjuuujjjkk11jjjggjjgjkljkuhu},
i.e. we have a consistency with
\er{MaxVacFullPPNmmmffffffiuiuhjuughbghhiuijghghghhjhjhhghyuyjjjhhjhjffkklkljljjkhhjkjhghghghl;kl;;k;kkljjkjkll;l;}
and
\er{MaxVacFullPPNmmmffffffiuiuhjuughbghhiuijghghghhhfhhghghguygtjuuujjjkk11}.

Next note that, as before, we can easily deduce that equations of
the ray propagation either of the form
\er{MaxVacFullPPNmmmffffffiuiuhjuughbghhiuijghghghhhfhhghghguygtjuuujjjkk11}
or of the form
\er{MaxVacFullPPNmmmffffffiuiuhjuughbghhiuijghghghhhfhhghghguygtjuuujjjkk11jjjggjjgtttjhhuhhjhjhgghghjhjggjugghgh}
are invariant under the change of non-inertial cartesian coordinate
system. Moreover,
\er{MaxVacFullPPNmmmffffffiuiuhjuughbghhiuijghghghhhfhhghghguygtjuuujjjkk11jjjggjjgkhhhhkjhkhkhhkkkhkkhhhhjggghjhj}
is also invariant under the change of non-inertial cartesian
coordinate system.

Next consider a wave characterized by:
\begin{equation}\label{MaxVacFullPPNmmmffffffiuiuhjuughbghhuiiujjhhjjhjhhj11}
U(\vec x,t)=A(\vec x,t)e^{ik_0S(\vec x,t)},
\end{equation}
and, consistently with
\er{MaxVacFullPPNmmmffffffiuiuhjuughbghhiuijghghghhhgjgfghjdfjgddg11}
consider a velocity field of the wave:
\begin{equation}\label{MaxVacFullPPNmmmffffffiuiuhjuughbghhiuijghghghhhfhhghghguygtjuuujjjkyuuykjkjj11}
\vec h(\vec x,t):=\vec {\tilde u}(\vec x,t)-c^2_0(\vec
x,t)\left(\frac{\partial S}{\partial t}(\vec x,t)+\vec {\tilde
u}(\vec x,t)\cdot\nabla_{\vec x}S(\vec x,t)\right)^{-1}\nabla_{\vec
x}S(\vec x,t),
\end{equation}
Furthermore, assume that the wave we consider undergoes reflection
and/or refraction on the time-dependent surface $\mathcal{T}$
having the outcoming three-dimensional unit normal $\vec n(\vec
x,t)$ and the motion velocity field $\vec w_\mathcal{T}(\vec x,t)$,
separating two regions characterized respectively by $c_0=c^{(1)}_0$
and $\vec {\tilde u}=\vec {\tilde u}_1$ and by $c^{(2)}_0$ and $\vec
{\tilde u}_2$, with the formation of the reflected wave,
characterized by:
\begin{equation}\label{MaxVacFullPPNmmmffffffiuiuhjuughbghhuiiujjhhjjhjhhjugh11}
U_1(\vec x,t)=A_1(\vec x,t)e^{ik_0S_1(\vec x,t)}\,,
\end{equation}
and by the velocity field:
\begin{equation}\label{MaxVacFullPPNmmmffffffiuiuhjuughbghhiuijghghghhhfhhghghguygtjuuujjjkyuuykjkjjhjhhj11}
\vec h_1(\vec x)=\vec {\tilde u}_1(\vec x,t)-c^2_0(\vec
x,t)\left(\frac{\partial S_1}{\partial t}(\vec x,t)+\vec {\tilde
u}_1(\vec x,t)\cdot\nabla_{\vec x}S_1(\vec
x,t)\right)^{-1}\nabla_{\vec x}S_1(\vec x,t),
\end{equation}
and formation of the refracted wave characterized by:
\begin{equation}\label{MaxVacFullPPNmmmffffffiuiuhjuughbghhuiiujjhhjjhjhhj111}
U_2(\vec x,t)=A_2(\vec x,t)e^{ik_0S_2(\vec x,t)}\,,
\end{equation}
and by the velocity field:
\begin{equation}\label{MaxVacFullPPNmmmffffffiuiuhjuughbghhiuijghghghhhfhhghghguygtjuuujjjkyuuykjkjj111}
\vec h_2(\vec x)=\vec {\tilde u}_2(\vec x,t)-(c^{(2)}_0)^2(\vec
x,t)\left(\frac{\partial S_2}{\partial t}(\vec x,t)+\vec {\tilde
u}_2(\vec x,t)\cdot\nabla_{\vec x}S_2(\vec
x,t)\right)^{-1}\nabla_{\vec x}S_2(\vec x,t).
\end{equation}
Then the boundary conditions of $U$, $U_1$ and $U_2$ depend on the
physical meaning of these fields. However, one of the
\underline{necessary} conditions should be that
\begin{equation}\label{MaxMedFullGGffgggyyojjyugggjhhjiiuuiyuyuyyuyuzzhhkkk11}
S_1(\vec x,t)=S_2(\vec x,t)+C_2=S(\vec x,t)\quad\quad\forall\vec
x\in\mathcal{T},\;\forall t,
\end{equation}
where $C_2$ is a real constant. In particular
\er{MaxMedFullGGffgggyyojjyugggjhhjiiuuiyuyuyyuyuzzhhkkk11} implies
\begin{equation}\label{MaxMedFullGGffgggyyojjyugggjhhjiiuuiyuyuyyuyuzzhhkkkgdfg11}
\nabla_{\vec x} S_1-\left(\vec n\cdot \nabla_{\vec x} S_1\right)\vec
n
=\nabla_{\vec x} S_2-\left(\vec n\cdot \nabla_{\vec x}
S_2\right)\vec n=\nabla_{\vec x} S-\left(\vec n\cdot \nabla_{\vec x}
S\right)\vec n\quad\quad\forall\vec x\in\mathcal{T},\;\forall t.
\end{equation}
In particular, for every point on the surface $\mathcal{T}$ vectors
$\nabla_{\vec x} S_1$ and $\nabla_{\vec x} S_2$ lie in the plane
formed by vectors $\vec n$ and $\nabla_{\vec x} S$. Moreover, by
\er{MaxMedFullGGffgggyyojjyugggjhhjiiuuiyuyuyyuyuzzhhkkk11} we also
have
\begin{equation}\label{MaxMedFullGGffgggyyojjyugggjhhjiiuuiyuyuyyuyuzzhhkkkgdfghgghgh11kjhk11}
\frac{\partial S}{\partial t}+\vec w_\mathcal{T}\cdot\nabla_{\vec x}
S=\frac{\partial S_1}{\partial t}+\vec
w_\mathcal{T}\cdot\nabla_{\vec x} S_1=\frac{\partial S_2}{\partial
t}+\vec w_\mathcal{T}\cdot\nabla_{\vec x} S_2\quad\quad\forall\vec
x\in\mathcal{T},\;\forall t.
\end{equation}
Finally, by
\er{MaxVacFullPPNmmmffffffiuiuhjuughbghhiuijghghghhjhjhhghyuyjjjhhjhjffkklkljljjkhhjkjhghghghl;kl;;k;kkljjkjkll;l;}
we have:
\begin{equation}\label{MaxVacFullPPNmmmffffffiuiuhjuughbghhiuijghghghhjhjhhghyuyjjjhhjhjffkklkljljjkhhjkjhghghghl;kl;;k;kkljjkjkll;l;jijhhjkjkkj}
\left|\vec h-\vec{\tilde u}\right|^2=c^2_0\,,\quad\left|\vec
h_1-\vec{\tilde u}\right|^2=c^2_0\quad\text{and}\quad\left|\vec
h_2-\vec{\tilde u}_2\right|^2=(c^{(2)}_0)^2.
\end{equation}
In particular, by
\er{MaxMedFullGGffgggyyojjyugggjhhjiiuuiyuyuyyuyuzzhhkkkgdfg11} and
\er{MaxMedFullGGffgggyyojjyugggjhhjiiuuiyuyuyyuyuzzhhkkkgdfghgghgh11kjhk11}
we have:
\begin{multline}\label{MaxMedFullGGffgggyyojjyugggjhhjiiuuiyuyuyyuyuzzhhkkkgdfg11lkjjjkk}
\nabla_{\vec x} S_1-\frac{1}{c^2_0}\left(\frac{\partial
S_1}{\partial t}+\vec w_{\mathcal{T}}\cdot\nabla_\vec x
S_1\right)\left(\vec {\tilde u}-\vec
w_{\mathcal{T}}\right)-\left(\vec n\cdot \left(\nabla_{\vec x}
S_1-\frac{1}{c^2_0}\left(\frac{\partial S_1}{\partial t}+\vec
w_{\mathcal{T}}\cdot\nabla_\vec x S_1\right)\left(\vec {\tilde
u}-\vec w_{\mathcal{T}}\right)\right)\right)\vec n
\\ =\left(\nabla_{\vec x} S-\frac{1}{c^2_0}\left(\frac{\partial
S}{\partial t}+\vec w_{\mathcal{T}}\cdot\nabla_\vec x
S\right)\left(\vec {\tilde u}-\vec
w_{\mathcal{T}}\right)\right)-\left(\vec n\cdot \left(\nabla_{\vec
x} S-\frac{1}{c^2_0}\left(\frac{\partial S}{\partial t}+\vec
w_{\mathcal{T}}\cdot\nabla_\vec x S\right)\left(\vec {\tilde u}-\vec
w_{\mathcal{T}}\right)\right)\right)\vec n
\\=
\left(\nabla_{\vec x} S_2-\frac{1}{c_0^2}\left(\frac{\partial
S_2}{\partial t}+\vec w_{\mathcal{T}}\cdot\nabla_\vec x
S_2\right)\left(\vec {\tilde u}-\vec
w_{\mathcal{T}}\right)\right)\\-\left(\vec n\cdot \left(\nabla_{\vec
x} S_2-\frac{1}{c_0^2}\left(\frac{\partial S_2}{\partial t}+\vec
w_{\mathcal{T}}\cdot\nabla_\vec x S_2\right)\left(\vec {\tilde
u}-\vec w_{\mathcal{T}}\right)\right)\right)\vec n
\quad\quad\forall\vec x\in\mathcal{T},\;\forall t.
\end{multline}

Next, assume that the following approximation is valid on the both
sides of the surface $\mathcal{T}$:
\begin{equation}\label{ojhkkkpkkljouu11}
\frac{|\vec w_\mathcal{T}|^2}{c^2_0+(c^{(2)}_0)^2}\ll
1\,,\quad\frac{|\vec {\tilde u}|^2}{c^2_0}\ll
1\quad\text{and}\quad\frac{|\vec {\tilde u}_2|^2}{(c^{(2)}_0)^2}\ll
1.
\end{equation}
Then, by \er{ojhkkkpkkljouu11} we approximate
\er{MaxVacFullPPNmmmffffffiuiuhjuughbghhiuijghghghhjhjhhghyuyjjjhhjhjffkkl}
as:
\begin{equation}\label{MaxVacFullPPNmmmffffffiuiuhjuughbghhiuijghghghhjhjhhghyuyjjjhhjhjffkklkljljjkhhjkjhjjkjk}
\frac{1}{c^2_0}\left(\frac{\partial S}{\partial
t}\right)^2\approx\left|\nabla_\vec x
S-\frac{1}{c^2_0}\left(\frac{\partial S}{\partial t}+\vec {\tilde
u}\cdot\nabla_\vec x S\right)\vec {\tilde u}\right|^2.
\end{equation}
Thus by \er{ojhkkkpkkljouu11} we further approximate
\er{MaxVacFullPPNmmmffffffiuiuhjuughbghhiuijghghghhjhjhhghyuyjjjhhjhjffkklkljljjkhhjkjhjjkjk}
as:
\begin{equation}\label{MaxVacFullPPNmmmffffffiuiuhjuughbghhiuijghghghhjhjhhghyuyjjjhhjhjffkklkljljjkhhjkjhjjkjkjkjkjk}
\frac{1}{c^2_0}\left(\frac{\partial S}{\partial t}+\vec
w_{\mathcal{T}}\cdot\nabla_\vec x S\right)^2\approx\left|\nabla_\vec
x S-\frac{1}{c^2_0}\left(\frac{\partial S}{\partial t}+\vec {\tilde
u}\cdot\nabla_\vec x S\right)\left(\vec {\tilde u}-\vec
w_{\mathcal{T}}\right)\right|^2.
\end{equation}
Then by
\er{MaxVacFullPPNmmmffffffiuiuhjuughbghhiuijghghghhhgjgfghjdfjgddg11}
we finally rewrite
\er{MaxVacFullPPNmmmffffffiuiuhjuughbghhiuijghghghhjhjhhghyuyjjjhhjhjffkklkljljjkhhjkjhjjkjkjkjkjk}
as:
\begin{equation}\label{MaxVacFullPPNmmmffffffiuiuhjuughbghhiuijghghghhjhjhhghyuyjjjhhjhjffkklkljljjkhhjkjhjjkjkjkjkjkjkkhk}
\left(\frac{\partial S}{\partial t}+\vec
w_{\mathcal{T}}\cdot\nabla_\vec x
S\right)^2\approx\left|\frac{1}{c_0}\left(\frac{\partial S}{\partial
t}+\vec {\tilde u}\cdot\nabla_\vec x S\right)\left(\vec h-\vec
w_{\mathcal{T}}\right)\right|^2.
\end{equation}
Similarly we obtain
\begin{equation}\label{MaxVacFullPPNmmmffffffiuiuhjuughbghhiuijghghghhjhjhhghyuyjjjhhjhjffkklkljljjkhhjkjhjjkjkjkjkjkjkkhklljggjop}
\left(\frac{\partial S_1}{\partial t}+\vec
w_{\mathcal{T}}\cdot\nabla_\vec x
S_1\right)^2\approx\left|\frac{1}{c_0}\left(\frac{\partial
S_1}{\partial t}+\vec {\tilde u}\cdot\nabla_\vec x
S_1\right)\left(\vec h_1-\vec w_{\mathcal{T}}\right)\right|^2,
\end{equation}
and
\begin{equation}\label{MaxVacFullPPNmmmffffffiuiuhjuughbghhiuijghghghhjhjhhghyuyjjjhhjhjffkklkljljjkhhjkjhjjkjkjkjkjkjkkhklljggjopjkjk}
\left(\frac{\partial S_2}{\partial t}+\vec
w_{\mathcal{T}}\cdot\nabla_\vec x
S_2\right)^2\approx\left|\frac{1}{c^{(2)}_0}\left(\frac{\partial
S_2}{\partial t}+\vec {\tilde u}_2\cdot\nabla_\vec x
S_2\right)\left(\vec h_2-\vec w_{\mathcal{T}}\right)\right|^2,
\end{equation}
In particular, by
\er{MaxVacFullPPNmmmffffffiuiuhjuughbghhiuijghghghhjhjhhghyuyjjjhhjhjffkklkljljjkhhjkjhjjkjkjkjkjkjkkhk},
\er{MaxVacFullPPNmmmffffffiuiuhjuughbghhiuijghghghhjhjhhghyuyjjjhhjhjffkklkljljjkhhjkjhjjkjkjkjkjkjkkhklljggjop},
\er{MaxVacFullPPNmmmffffffiuiuhjuughbghhiuijghghghhjhjhhghyuyjjjhhjhjffkklkljljjkhhjkjhjjkjkjkjkjkjkkhklljggjopjkjk}
and
\er{MaxMedFullGGffgggyyojjyugggjhhjiiuuiyuyuyyuyuzzhhkkkgdfghgghgh11kjhk11}
we deduce:
\begin{multline}\label{MaxVacFullPPNmmmffffffiuiuhjuughbghhiuijghghghhjhjhhghyuyjjjhhjhjffkklkljljjkhhjkjhjjkjkjkjkjkjkkhklljggjopjkjkklljljggh}
\frac{(c^{(2)}_0)^2}{c_0^2}\left|\frac{c_0}{(c^{(2)}_0)^2}\left(\frac{\partial
S_2}{\partial t}+\vec {\tilde u}_2\cdot\nabla_\vec x
S_2\right)\left(\vec h_2-\vec
w_{\mathcal{T}}\right)\right|^2\approx\\
\left|\frac{1}{c_0}\left(\frac{\partial S_1}{\partial t}+\vec
{\tilde u}\cdot\nabla_\vec x S_1\right)\left(\vec h_1-\vec
w_{\mathcal{T}}\right)\right|^2 \approx
\left|\frac{1}{c_0}\left(\frac{\partial S}{\partial t}+\vec {\tilde
u}\cdot\nabla_\vec x S\right)\left(\vec h-\vec
w_{\mathcal{T}}\right)\right|^2.
\end{multline}
Moreover, by \er{ojhkkkpkkljouu11} we approximate
\er{MaxMedFullGGffgggyyojjyugggjhhjiiuuiyuyuyyuyuzzhhkkkgdfg11lkjjjkk}
as:
\begin{multline}\label{MaxMedFullGGffgggyyojjyugggjhhjiiuuiyuyuyyuyuzzhhkkkgdfg11lkjjjkkhjhjhklljkjkl}
\nabla_{\vec x} S_1-\frac{1}{c^2_0}\left(\frac{\partial
S_1}{\partial t}+\vec {\tilde u}\cdot\nabla_\vec x
S_1\right)\left(\vec {\tilde u}-\vec
w_{\mathcal{T}}\right)-\left(\vec n\cdot \left(\nabla_{\vec x}
S_1-\frac{1}{c^2_0}\left(\frac{\partial S_1}{\partial t}+\vec
{\tilde u}\cdot\nabla_\vec x S_1\right)\left(\vec {\tilde u}-\vec
w_{\mathcal{T}}\right)\right)\right)\vec n
=\\ \left(\nabla_{\vec x} S-\frac{1}{c^2_0}\left(\frac{\partial
S}{\partial t}+\vec {\tilde u}\cdot\nabla_\vec x S\right)\left(\vec
{\tilde u}-\vec w_{\mathcal{T}}\right)\right)-\left(\vec n\cdot
\left(\nabla_{\vec x} S-\frac{1}{c^2_0}\left(\frac{\partial
S}{\partial t}+\vec {\tilde u}\cdot\nabla_\vec x S\right)\left(\vec
{\tilde u}-\vec w_{\mathcal{T}}\right)\right)\right)\vec n
\\=
\left(\nabla_{\vec x}
S_2-\frac{1}{(c^{(2)}_0)^2}\left(\frac{\partial S_2}{\partial
t}+\vec {\tilde u}_2\cdot\nabla_\vec x S_2\right)\left(\vec {\tilde
u}_2-\vec w_{\mathcal{T}}\right)\right)\\-\left(\vec n\cdot
\left(\nabla_{\vec x}
S_2-\frac{1}{(c^{(2)}_0)^2}\left(\frac{\partial S_2}{\partial
t}+\vec {\tilde u}_2\cdot\nabla_\vec x S_2\right)\left(\vec {\tilde
u}_2-\vec
w_{\mathcal{T}}\right)\right)\right)\vec n\\
+\frac{1}{c_0^2}\left(\frac{\partial S_2}{\partial t}+\vec {\tilde
u}_2\cdot\nabla_\vec x
S_2\right)\left(\frac{c^2_0}{(c^{(2)}_0)^2}\vec {\tilde u}_2-\vec
{\tilde u}-\left(\frac{c^2_0}{(c^{(2)}_0)^2}-1\right)\vec
w_{\mathcal{T}}\right)\\-\frac{1}{c_0^2}\left(\frac{\partial
S_2}{\partial t}+\vec {\tilde u}_2\cdot\nabla_\vec x
S_2\right)\left(\vec n\cdot \left(\frac{c^2_0}{(c^{(2)}_0)^2}\vec
{\tilde u}_2-\vec {\tilde
u}-\left(\frac{c^2_0}{(c^{(2)}_0)^2}-1\right)\vec
w_{\mathcal{T}}\right)\right)\vec n
\quad\quad\forall\vec x\in\mathcal{T},\;\forall t,
\end{multline}
and then by
\er{MaxVacFullPPNmmmffffffiuiuhjuughbghhiuijghghghhhgjgfghjdfjgddg11}
rewrite it as:
\begin{multline}\label{MaxMedFullGGffgggyyojjyugggjhhjiiuuiyuyuyyuyuzzhhkkkgdfg11lkjjjkkhjhjhggojoj}
\frac{1}{c_0}\left(\frac{\partial S_1}{\partial t}+\vec {\tilde
u}\cdot\nabla_\vec x S_1\right)\left(\vec h_1-\vec
w_{\mathcal{T}}\right)- \frac{1}{c_0}\left(\frac{\partial
S_1}{\partial t}+\vec {\tilde u}\cdot\nabla_\vec x
S_1\right)\left(\vec n\cdot\left(\vec h_1-\vec
w_{\mathcal{T}}\right)\right)\vec n=\\
\frac{1}{c_0}\left(\frac{\partial S}{\partial t}+\vec {\tilde
u}\cdot\nabla_\vec x S\right)\left(\vec h-\vec
w_{\mathcal{T}}\right)- \frac{1}{c_0}\left(\frac{\partial
S}{\partial t}+\vec {\tilde u}\cdot\nabla_\vec x S\right)\left(\vec
n\cdot\left(\vec h-\vec w_{\mathcal{T}}\right)\right)\vec n\\
=\frac{c_0}{(c^{(2)}_0)^2}\left(\frac{\partial S_2}{\partial t}+\vec
{\tilde u}_2\cdot\nabla_\vec x S_2\right)\left(\vec h_2-\vec
w_{\mathcal{T}}\right)-
\frac{c_0}{(c^{(2)}_0)^2}\left(\frac{\partial S_2}{\partial t}+\vec
{\tilde u}_2\cdot\nabla_\vec x S_2\right)\left(\vec n\cdot\left(\vec
h_2-\vec w_{\mathcal{T}}\right)\right)\vec n \\
-\frac{1}{c_0}\left(\frac{\partial S_2}{\partial t}+\vec {\tilde
u}_2\cdot\nabla_\vec x
S_2\right)\left(\frac{c^2_0}{(c^{(2)}_0)^2}\vec {\tilde u}_2-\vec
{\tilde u}-\left(\frac{c^2_0}{(c^{(2)}_0)^2}-1\right)\vec
w_{\mathcal{T}}\right)\\+\frac{1}{c_0}\left(\frac{\partial
S_2}{\partial t}+\vec {\tilde u}_2\cdot\nabla_\vec x
S_2\right)\left(\vec n\cdot \left(\frac{c^2_0}{(c^{(2)}_0)^2}\vec
{\tilde u}_2-\vec {\tilde
u}-\left(\frac{c^2_0}{(c^{(2)}_0)^2}-1\right)\vec
w_{\mathcal{T}}\right)\right)\vec n \quad\quad\forall\vec
x\in\mathcal{T},\;\forall t.
\end{multline}
Then, since the directions of vectors $\vec h$ and $\vec h_1$ should
be different, by
\er{MaxMedFullGGffgggyyojjyugggjhhjiiuuiyuyuyyuyuzzhhkkkgdfg11lkjjjkkhjhjhggojoj}
and
\er{MaxVacFullPPNmmmffffffiuiuhjuughbghhiuijghghghhjhjhhghyuyjjjhhjhjffkklkljljjkhhjkjhjjkjkjkjkjkjkkhklljggjopjkjkklljljggh}
we deduce
\begin{equation}\label{MaxMedFullGGffgggyyojjyugggjhhjiiuuiyuyuyyuyuzzhhkkkgdfghgghghijhjhhkh11}
\left(\frac{\partial S_1}{\partial t}+\vec {\tilde
u}\cdot\nabla_\vec x S_1\right)\left(\vec n\cdot\left(\vec h_1-\vec
w_{\mathcal{T}}\right)\right)=-\left(\frac{\partial S}{\partial
t}+\vec {\tilde u}\cdot\nabla_\vec x S\right)\left(\vec
n\cdot\left(\vec h-\vec
w_{\mathcal{T}}\right)\right)\quad\quad\forall\vec x\in\mathcal{T}.
\end{equation}
So, by
\er{MaxMedFullGGffgggyyojjyugggjhhjiiuuiyuyuyyuyuzzhhkkkgdfg11lkjjjkkhjhjhggojoj}
and
\er{MaxMedFullGGffgggyyojjyugggjhhjiiuuiyuyuyyuyuzzhhkkkgdfghgghghijhjhhkh11}
we obtain the law of reflection: vector $(\vec h_1-\vec
w_{\mathcal{T}})$ lies in the plane formed by vectors $\vec n$ and
$(\vec h-\vec w_{\mathcal{T}})$, and we have:
\begin{equation}\label{MaxMedFulljhhjjjjj11}
\theta\left((\vec h-\vec w_{\mathcal{T}}),-\vec
n\right)=\theta_1\left((\vec h_1-\vec w_{\mathcal{T}}),\vec n\right)
\end{equation}
where $\theta\left((\vec h-\vec w_{\mathcal{T}}),-\vec n\right)$ is
the angle between the vector of the relative velocity of the
incoming ray, relative to the surface of reflection, $(\vec h-\vec
w_{\mathcal{T}})$ and the incoming normal to the surface $-\vec n$
and $\theta_1\left((\vec h_1-\vec w_{\mathcal{T}}),\vec n\right)$ is
the angle between the vector of the relative velocity of the
reflected ray, relative to the surface of reflection, $(\vec
h_1-\vec w_{\mathcal{T}})$ and the outcoming normal $\vec n$. Note
that, since $\vec h$ and $\vec w_{\mathcal{T}}$ are both speed like
vector fields then $(\vec h-\vec w_{\mathcal{T}})$ is a
\underline{proper} vector field and thus the above law of reflection
together with \er{MaxMedFulljhhjjjjj11} are invariant under the
change of inertial or non-inertial cartesian coordinate systems. In
particular, if \er{ojhkkkpkkljouu11} holds for some cartesian
coordinate system, then we can use this law also in other coordinate
systems where \er{ojhkkkpkkljouu11} does not hold. Therefore, for
the validity of the above law of reflection we may assume the
following relation instead of \er{ojhkkkpkkljouu11}:
\begin{equation}\label{ojhkkkpkkljouu11uihkhkjkjk}
\frac{|\vec {\tilde u}-\vec w_\mathcal{T}|^2}{c^2_0}\ll 1.
\end{equation}

Next, assume that the wave we consider in
\er{MaxVacFullPPNmmmffffffiuiuhjuughbghhuiiujjhhjjhjhhj11} has an
electromagnetic nature. Then by \er{gughhghfbvnbvyyuuyr} we have
\begin{equation}\label{gughhghfbvnbvyyuuy11}
c_0=c\sqrt{\kappa_0\gamma_0}\quad\text{and}\quad\vec {\tilde
u}=\left(\gamma_0\kappa_0\vec v+(1-\gamma_0\kappa_0)\vec u\right),
\end{equation}
where, $\vec u$ is the medium velocity and $\vec v$ is the local
vectorial gravitational potential. Similarly, on the second side of
surface $\mathcal{T}$ we have
\begin{equation}\label{gughhghfbvnbvyyuuyGGHHG11}
c^{(2)}_0=c\sqrt{\kappa^{(2)}_0\gamma^{(2)}_0}\quad\text{and}\quad\vec
{\tilde u}_2=\left(\gamma^{(2)}_0\kappa^{(2)}_0\vec
v+(1-\gamma^{(2)}_0\kappa^{(2)}_0)\vec u_2\right),
\end{equation}
where, $\vec u_2$ is the medium velocity on the second side of
surface $\mathcal{T}$. Then we rewrite the first equalities in
\er{gughhghfbvnbvyyuuy11} and \er{gughhghfbvnbvyyuuyGGHHG11} as:
\begin{equation}\label{gughhghfbvnbvyyuuyggg11}
c_0=c\sqrt{\kappa_0\gamma_0}\quad\text{and}\quad
c^{(2)}_0=c\sqrt{\kappa^{(2)}_0\gamma^{(2)}_0}.
\end{equation}
Then in particular, by \er{gughhghfbvnbvyyuuyggg11} we deduce
\begin{equation}\label{gughhghfbvnbvyyuuygggyityutt}
\frac{c^2_0}{(c^{(2)}_0)^2}=\frac{\gamma_0\kappa_0}{\gamma^{(2)}_0\kappa^{(2)}_0}.
\end{equation}
Thus by \er{gughhghfbvnbvyyuuygggyityutt}, \er{gughhghfbvnbvyyuuy11}
and \er{gughhghfbvnbvyyuuyGGHHG11} we deduce:
\begin{multline}\label{MaxMedFullGGffgggyyojjyugggjhhjiiuuiyuyuyyuyuzzhhkkkgdfg11lkjjjkkhjhjhggojojjiijhhihujkhk}
\left(\frac{c^2_0}{(c^{(2)}_0)^2}\vec {\tilde u}_2-\vec {\tilde
u}-\left(\frac{c^2_0}{(c^{(2)}_0)^2}-1\right)\vec
w_{\mathcal{T}}\right)=\\
\left(\frac{\gamma_0\kappa_0}{\gamma^{(2)}_0\kappa^{(2)}_0}\left(\gamma^{(2)}_0\kappa^{(2)}_0\vec
v+(1-\gamma^{(2)}_0\kappa^{(2)}_0)\vec
u_2\right)-\left(\gamma_0\kappa_0\vec v+(1-\gamma_0\kappa_0)\vec
u\right)-\left(\frac{\gamma_0\kappa_0}{\gamma^{(2)}_0\kappa^{(2)}_0}-1\right)\vec
w_{\mathcal{T}}\right)
=\\
\left(\frac{\gamma_0\kappa_0}{\gamma^{(2)}_0\kappa^{(2)}_0}-1\right)\left(\vec
u_2-\vec w_{\mathcal{T}}\right)+(1-\gamma_0\kappa_0)\left(\vec
u_2-\vec u\right).
\end{multline}
On the other hand, since $\vec u$ and $\vec u_2$ are velocities of
the medium on both sides of the surface $\mathcal{T}$ and $\vec
w_{\mathcal{T}}$ is the velocity of the surface, we have equalities
of these three vectors on the surface:
\begin{equation}\label{MaxMedFullGGffgggyyojjyugggjhhjiiuuiyuyuyyuyuzzhhkkkgdfg111ioouiu}
\vec u_2(\vec x,t)
=\vec u(\vec x,t)
=\vec w_{\mathcal{T}}(\vec x,t)
\quad\quad\forall\vec x\in\mathcal{T},\;\forall t.
\end{equation}
Thus with the help of
\er{MaxMedFullGGffgggyyojjyugggjhhjiiuuiyuyuyyuyuzzhhkkkgdfg11lkjjjkkhjhjhggojojjiijhhihujkhk}
and
\er{MaxMedFullGGffgggyyojjyugggjhhjiiuuiyuyuyyuyuzzhhkkkgdfg111ioouiu}
we infer:
\begin{multline}\label{MaxMedFullGGffgggyyojjyugggjhhjiiuuiyuyuyyuyuzzhhkkkgdfg11lkjjjkkhjhjhggojojiuuii}
\frac{1}{c_0}\left(\frac{\partial S_2}{\partial t}+\vec {\tilde
u}_2\cdot\nabla_\vec x
S_2\right)\left(\frac{c^2_0}{(c^{(2)}_0)^2}\vec {\tilde u}_2-\vec
{\tilde u}-\left(\frac{c^2_0}{(c^{(2)}_0)^2}-1\right)\vec
w_{\mathcal{T}}\right)\\-\frac{1}{c_0}\left(\frac{\partial
S_2}{\partial t}+\vec {\tilde u}_2\cdot\nabla_\vec x
S_2\right)\left(\vec n\cdot \left(\frac{c^2_0}{(c^{(2)}_0)^2}\vec
{\tilde u}_2-\vec {\tilde
u}-\left(\frac{c^2_0}{(c^{(2)}_0)^2}-1\right)\vec
w_{\mathcal{T}}\right)\right)\vec n =0\quad\quad\forall\vec
x\in\mathcal{T},\;\forall t.
\end{multline}
Note that if there is a vacuum on the one side of surface
$\mathcal{T}$ then we can obtain
\er{MaxMedFullGGffgggyyojjyugggjhhjiiuuiyuyuyyuyuzzhhkkkgdfg11lkjjjkkhjhjhggojojiuuii}
from
\er{MaxMedFullGGffgggyyojjyugggjhhjiiuuiyuyuyyuyuzzhhkkkgdfg11lkjjjkkhjhjhggojojjiijhhihujkhk}
with even easier proof. Therefore, using
\er{MaxMedFullGGffgggyyojjyugggjhhjiiuuiyuyuyyuyuzzhhkkkgdfg11lkjjjkkhjhjhggojojiuuii}
by
\er{MaxMedFullGGffgggyyojjyugggjhhjiiuuiyuyuyyuyuzzhhkkkgdfg11lkjjjkkhjhjhggojoj}
we deduce
\begin{multline}\label{MaxMedFullGGffgggyyojjyugggjhhjiiuuiyuyuyyuyuzzhhkkkgdfg11lkjjjkkhjhjhggojojjhjggjhkh}
\frac{1}{c_0}\left(\frac{\partial S_1}{\partial t}+\vec {\tilde
u}\cdot\nabla_\vec x S_1\right)\left(\vec h_1-\vec
w_{\mathcal{T}}\right)- \frac{1}{c_0}\left(\frac{\partial
S_1}{\partial t}+\vec {\tilde u}\cdot\nabla_\vec x
S_1\right)\left(\vec n\cdot\left(\vec h_1-\vec
w_{\mathcal{T}}\right)\right)\vec n=\\
\frac{c_0}{(c^{(2)}_0)^2}\left(\frac{\partial S_2}{\partial t}+\vec
{\tilde u}_2\cdot\nabla_\vec x S_2\right)\left(\vec h_2-\vec
w_{\mathcal{T}}\right)-
\frac{c_0}{(c^{(2)}_0)^2}\left(\frac{\partial S_2}{\partial t}+\vec
{\tilde u}_2\cdot\nabla_\vec x S_2\right)\left(\vec n\cdot\left(\vec
h_2-\vec w_{\mathcal{T}}\right)\right)\vec n\\ \quad\quad\forall\vec
x\in\mathcal{T},\;\forall t,
\end{multline}
and in particular, for every point on the surface $\mathcal{T}$
vector $(\vec h_2-\vec w_{\mathcal{T}})$ lies in the plane formed by
vectors $\vec n$ and $(\vec h-\vec w_{\mathcal{T}})$. On the other
hand, by
\er{MaxVacFullPPNmmmffffffiuiuhjuughbghhiuijghghghhjhjhhghyuyjjjhhjhjffkklkljljjkhhjkjhjjkjkjkjkjkjkkhklljggjopjkjkklljljggh}
we have
\begin{multline}\label{MaxVacFulfgjhgikhkouoljjk}
\frac{(c^{(2)}_0)^2}{c_0^2}\left|\frac{c_0}{(c^{(2)}_0)^2}\left(\frac{\partial
S_2}{\partial t}+\vec {\tilde u}_2\cdot\nabla_\vec x
S_2\right)\left(\vec h_2-\vec w_{\mathcal{T}}\right)\right|^2\approx
\left|\frac{1}{c_0}\left(\frac{\partial S_1}{\partial t}+\vec
{\tilde u}\cdot\nabla_\vec x S_1\right)\left(\vec h_1-\vec
w_{\mathcal{T}}\right)\right|^2.
\end{multline}
Thus by
\er{MaxMedFullGGffgggyyojjyugggjhhjiiuuiyuyuyyuyuzzhhkkkgdfg11lkjjjkkhjhjhggojojjhjggjhkh}
and \er{MaxVacFulfgjhgikhkouoljjk}, exactly as in the proof of
\er{MaxMedFulljhhjjjjjhhhj}, we finally deduce that we have the
following Snell's law of refraction: $(\vec h_2-\vec
w_{\mathcal{T}})$ lies in the plane formed by vectors $\vec n$ and
$(\vec h-\vec w_{\mathcal{T}})$ and we have:
\begin{equation}\label{MaxMedFulljhhjjjjjhhhj11}
n\sin{\left(\theta\left((\vec h-\vec w_{\mathcal{T}}),\vec
n\right)\right)}=n_2\sin{\left(\theta_2\left((\vec h_2-\vec
w_{\mathcal{T}}),\vec n\right)\right)}
\end{equation}
where $\theta\left((\vec h-\vec w_{\mathcal{T}}),\vec n\right)$ is
the angle between the vector of the relative velocity of the
incoming ray, relative to the surface of refraction, $(\vec h-\vec
w_{\mathcal{T}})$ and the normal to the surface $\vec n$,
$\theta_2\left((\vec h_2-\vec w_{\mathcal{T}}),\vec n\right)$ is the
angle between the vector of the relative velocity of the refracted
ray, relative to the surface of refraction, $(\vec h_2-\vec
w_{\mathcal{T}})$ and the normal $\vec n$ and as in
\er{MaxMedFullGGffgggyyojjhhjkhjyyiuhggjhhjhuyytytyuuytrrtghjtyuggyuighjuyioyyfgffhyuhhghzzrrkkhhkkkhhhjhkjhhghhgghiuiu1}
we set refraction indexes:
\begin{equation}\label{MaxMedFullGGffgggyyojjhhjkhjyyiuhggjhhjhuyytytyuuytrrtghjtyuggyuighjuyioyyfgffhyuhhghzzrrkkhhkkkhhhjhkjhhghhgghiuiujjkjk11}
n:=\frac{c}{c_0}\quad\text{and}\quad n_2:=\frac{c}{c^{(2)}_0}.
\end{equation}
%
%
%
Note that as the above law of reflection the Snell's law together
with \er{MaxMedFulljhhjjjjjhhhj11} are invariant under the change of
inertial or non-inertial cartesian coordinate systems. In
particular, if \er{ojhkkkpkkljouu11} holds for some cartesian
coordinate system, then we can use this law also in other coordinate
systems where \er{ojhkkkpkkljouu11} does not hold. Therefore, as
before, for the validity of the above Snell's law we may assume the
following relation instead of \er{ojhkkkpkkljouu11}:
\begin{equation}\label{ojhkkkpkkljouu11uihkhkjkjkkliyuttyty}
\frac{|\vec {\tilde u}-\vec w_\mathcal{T}|^2}{c^2_0}\ll
1\quad\text{and}\quad\frac{|\vec {\tilde u}_2-\vec
w_\mathcal{T}|^2}{(c^{(2)}_0)^2}\ll 1.
\end{equation}

\subsubsection{Polarization of the light inside a moving medium and/or in the presence of gravitational
field}\label{jgjghjghhffhf}
Again consider the system of Maxwell equations in the vacuum or in a
medium, having the form \er{MaxMedFullGGffykkhjhhzz} in the absence
of macroscopic charges and/or currents:
\begin{equation}\label{MaxVacFullPPNffGGhjhjhjh}
\begin{cases}
curl_{\vec x} \vec H=
\frac{1}{c}\frac{\partial \vec D}{\partial t},\\
div_{\vec x} \vec D=0
,\\
curl_{\vec x} \vec E+\frac{1}{c}\frac{\partial \vec B}{\partial t}=0,\\
div_{\vec x} \vec B=0,\\
\vec E=\gamma_0\vec D-\frac{1}{c}\,\vec {\tilde u}\times \vec B,\\
\vec H=\kappa_0\vec B+\frac{1}{c}\,\vec {\tilde u}\times \vec D,
\\ \vec {\tilde u}=\left(\gamma_0\kappa_0\vec v+(1-\gamma_0\kappa_0)\vec u\right),
\end{cases}
\end{equation}
and consider the case of monochromatic wave of the constant
frequency $\nu=\frac{\omega}{2\pi}$ where the fields $\vec {\tilde
u}$, $\gamma_0$ and $\kappa_0$ are independent of the time variable.
We also again assume that either our medium has no dispersion or
$\vec u\equiv 0$ (and so $\vec {\tilde u}\equiv\vec v$) and the
medium is transparent for the given frequency
$\nu=\frac{\omega}{2\pi}$. Next assume the \underline{rough}
Geometric Optics approximation \er{slochangGGaaffgfgjhjggh11}
(stronger than \er{slochangGGaaffgfgjhjggh}) that means the
following: assume that the changes in time of the basic
characteristics of the electromagnetic field become essential after
certain interval of time $T_e$ and the changes in space of of the
basic characteristics of the electromagnetic field become essential
in the spatial landscape $L_e$. Then we assume that
\begin{equation}\label{slochangGGaaffgfgjhjgghhkjhklkllk}
k_0c_0T_e\,\gg\, 1\quad\text{and}\quad k_0L_e\,\gg\, 1,
\end{equation}
where
\begin{equation}\label{MaxVacFullPPNmmmffffffiuiuhjuughbghhuiiujjhhjjhjhhjyuyhhjjjkjk}
k_0=\frac{\omega}{c}.
\end{equation}
We also assume \er{ojhkk}:
\begin{equation}\label{ojhkkkhhhh}
\frac{|\vec {\tilde u}|^2}{c^2_0}\ll 1.
\end{equation}
Then consistently with
\er{MaxVacFullPPNmmmffffffiuiuhjuughbghhuiiujj} we can write
\begin{equation}\label{MaxVacFullPPNmmmffffffiuiuhjuughbghhuiiujjnewjikhkjhjkkjhhjkkhhhj}
\begin{cases}
\vec D(\vec x,t)=\Xi_1\cdot\vec D_a(\vec x)e^{ik_0S(\vec x,t)}
\\
\vec B(\vec x,t)=\Xi_2\cdot\vec B_a(\vec x)e^{ik_0S(\vec x,t)}
\\
\vec E(\vec x,t)=\Xi_3\cdot\vec E_a(\vec x)e^{ik_0S(\vec x,t)}
\\
\vec H(\vec x,t)=\Xi_4\cdot\vec H_a(\vec x)e^{ik_0S(\vec x,t)}
\\
\vec E=\gamma_0\vec D-\frac{1}{c}\,\vec {\tilde u}\times \vec B\\
\vec H=\kappa_0\vec B+\frac{1}{c}\,\vec {\tilde u}\times \vec D.
\end{cases}
\end{equation}
Here $\vec D_a(\vec x),\vec B_a(\vec x),\vec E_a(\vec x),\vec
H_a(\vec x):\mathbb{R}^3\to\mathbb{R}^3$ are real vector fields,
independent of the time variable, $S(\vec x,t)$ is a real function
such that $\frac{\partial S}{\partial t}=c$ and
$\Xi_1,\Xi_2,\Xi_3,\Xi_4\in \mathbb{C}^{3\times 3}$ are constant
complex diagonal matrices of the form:
\begin{equation}\label{ghgyfftdtdljkljnnewkkkkk}
\Xi_k=\left(\begin{matrix}e^{i\theta_{1k}}&0&0\\0&e^{i\theta_{2k}}&0\\0&0&e^{i\theta_{3k}}\end{matrix}\right)\quad\quad\quad\quad\forall
k=1,2,3,4,
\end{equation}
where $\theta_{1k},\theta_{2k},\theta_{3k}$ are real constants. Then
denoting:
\begin{equation}\label{MaxVacFullPPNmmmffffffiuiuhjuughbghhuiiujjnewjikhkjhjkkjhhjkkhhhjhkhjh}
\begin{cases}
\vec D_1(\vec x)=\Xi_1\cdot\vec D_a(\vec x)
\\
\vec B_1(\vec x)=\Xi_2\cdot\vec B_a(\vec x)
\\
\vec E_1(\vec x)=\Xi_3\cdot\vec E_a(\vec x)
\\
\vec H_1(\vec x)=\Xi_4\cdot\vec H_a(\vec x),
\end{cases}
\end{equation}
we deduce:
\begin{equation}\label{MaxVacFullPPNmmmffffffiuiuhjuughbghhuiiujjnewjikhkjhjkkjhh}
\begin{cases}
\vec D(\vec x,t)=\vec D_1(\vec x)e^{ik_0S(\vec x,t)}
\\
\vec B(\vec x,t)=\vec B_1(\vec x)e^{ik_0S(\vec x,t)}
\\
\vec E(\vec x,t)=\vec E_1(\vec x)e^{ik_0S(\vec x,t)}
\\
\vec H(\vec x,t)=\vec H_1(\vec x)e^{ik_0S(\vec x,t)}
\\
\vec E_1=\gamma_0\vec D_1-\frac{1}{c}\,\vec {\tilde u}\times \vec B_1\\
\vec H_1=\kappa_0\vec B_1+\frac{1}{c}\,\vec {\tilde u}\times \vec
D_1.
\end{cases}
\end{equation}
Note that $\vec D_1,\vec B_1,\vec E_1,\vec H_1$ are complex. Then,
consistently with
\er{MaxVacFullPPNmmmffffffiuiuhjuughbghhiuijghghghhjhjhhghyuyiyyujjljk},
$S$ satisfies the Eikonal equation:
\begin{equation}\label{MaxVacFullPPNmmmffffffiuiuhjuughbghhiuijghghghhhfhhghghguygtjuuujjjknewljjljllkgugghklkllljojjj}
\left|\frac{c\vec {\tilde u}}{c^2_0}-\nabla_\vec x
S\right|^2=\frac{c^2}{c^2_0},
\end{equation}
and consistently with
\er{MaxVacFullPPNmmmffffffiuiuhjuughbghhiuijghghghhhfhhghghguygtjuuujjjk}
if $\vec A_1$ denotes either one of the vectors $\vec D_a,\vec
B_a,\vec E_a,\vec H_a$ or one of the vectors $\vec D_1,\vec B_1,\vec
E_1,\vec H_1$ then:
\begin{equation}\label{MaxVacFullPPNmmmffffffiuiuhjuughbghhiuijghghghhhfhhghghguygtjuuujjjknewljjljllkgugghklkllljjkjkkkkl}
\left\{d_{\vec x}\vec A_1\right\}^T\cdot\left(\frac{c\vec {\tilde
u}}{c^2_0}-\nabla_{\vec x}S\right)+\frac{1}{2}\left(-\Delta_{\vec
x}S\right)\vec A_1=0.
\end{equation}
Moreover, consistently with
\er{MaxVacFullPPNmmmffffffiuiuhjuughbghhiuijghghghhjhjhhghyuyhhjhjihikjjjjkjljkl},
\er{MaxVacFullPPNmmmffffffiuiuhjuughbghhiuijghghghhjhjhhghyuyhhjhjihikjjj}
and \er{gughhghfbvnbvyyuuyr} we have:
\begin{equation}\label{MaxVacFullPPNmmmffffffiuiuhjuughbghhiuijghghghhjhjhhghyuyhhjhjihikjjjnew}
\begin{cases}
c_0=c\sqrt{\kappa_0\gamma_0}\\
\frac{\partial S}{\partial t}=c\\
\frac{\partial \vec A_1}{\partial t}=0,\\
\frac{\partial \vec {\tilde u}}{\partial t}=0\\
\frac{\partial c_0}{\partial t}=0,
\end{cases}
\end{equation}
and, consistently with
\er{MaxVacFullPPNmmmffffffiuiuhjuughbghhiuijghghghhhfhhghghguygtjuuujjjkyuuy},
the vector field $\vec h_0$ defined for every $\vec x$ by:
\begin{equation}\label{MaxVacFullPPNmmmffffffiuiuhjuughbghhiuijghghghhhfhhghghguygtjuuujjjkyuuykkhhk}
\vec h_0(\vec x):=\frac{c}{c^2_0(\vec x)}\vec {\tilde u}(\vec
x)-\nabla_{\vec x}S(\vec x)\approx\frac{c}{c^2_0(\vec x)}\vec h(\vec
x) ,
\end{equation}
is the direction of the propagation of the ray that passes through
point $\vec x$.

Next, by inserting
\er{MaxVacFullPPNmmmffffffiuiuhjuughbghhuiiujjnewjikhkjhjkkjhh} into
\er{MaxVacFullPPNffGGhjhjhjh}, using the fact that $\frac{\partial
S}{\partial t}=c$ and using the rough approximation
\er{slochangGGaaffgfgjhjgghhkjhklkllk} together with \er{apfrm5} we
obtain:
\begin{equation}\label{MaxVacFullPPNffGGhjhjhjhjkhkh}
\begin{cases}
\vec D_1\approx \nabla_{\vec x}S\times \vec H_1,\\
\nabla_{\vec x}S\cdot \vec D_1\approx 0
,\\
\vec B_1\approx-\nabla_{\vec x}S\times  \vec E_1,\\
\nabla_{\vec x}S\cdot \vec B_1\approx 0.
\end{cases}
\end{equation}
Thus, by \er{MaxVacFullPPNffGGhjhjhjhjkhkh} and the last two
equations in
\er{MaxVacFullPPNmmmffffffiuiuhjuughbghhuiiujjnewjikhkjhjkkjhh},
using \er{apfrm2} we obtain:
\begin{multline}\label{MaxVacFullPPNffGGhjhjhjhjkhkhljjkjn}
\vec D_1\approx \nabla_{\vec x}S\times \left(\kappa_0\vec
B_1+\frac{1}{c}\,\vec {\tilde u}\times \vec
D_1\right)=\kappa_0\nabla_{\vec x}S\times\vec
B_1-\frac{1}{c}\left(\vec {\tilde u}\cdot\nabla_{\vec x}S\right)\vec D_1\quad\text{and}\\
\vec B_1\approx-\nabla_{\vec x}S\times  \left(\gamma_0\vec
D_1-\frac{1}{c}\,\vec {\tilde u}\times \vec
B_1\right)=-\gamma_0\nabla_{\vec x}S\times\vec
D_1-\frac{1}{c}\left(\vec {\tilde u}\cdot\nabla_{\vec x}S\right)\vec
B_1.
\end{multline}
So
\begin{multline}\label{MaxVacFullPPNffGGhjhjhjhjkhkhljjkjnhhhgjgjg}
\vec D_1\approx\frac{\kappa_0}{\left(1+\frac{1}{c}\vec {\tilde
u}\cdot\nabla_{\vec x}S\right)}\nabla_{\vec x}S\times\vec
B_1\quad\quad\text{and}\quad\quad \vec
B_1\approx-\frac{\gamma_0}{\left(1+\frac{1}{c}\vec {\tilde
u}\cdot\nabla_{\vec x}S\right)}\nabla_{\vec x}S\times\vec D_1.
\end{multline}
Thus by inserting \er{MaxVacFullPPNffGGhjhjhjhjkhkhljjkjnhhhgjgjg}
and \er{MaxVacFullPPNffGGhjhjhjhjkhkh} into
\er{MaxVacFullPPNmmmffffffiuiuhjuughbghhuiiujjnewjikhkjhjkkjhh} we
infer:
\begin{equation}\label{MaxVacFullPPNffGGhjhjhjhjkhkhklklk}
\begin{cases}
 \vec
B\approx\frac{\gamma_0}{\left(1+\frac{1}{c}\vec {\tilde
u}\cdot\nabla_{\vec x}S\right)}\left(-\nabla_{\vec
x}S\right)\times\vec D
\\
\vec D\approx-\frac{\kappa_0}{\left(1+\frac{1}{c}\vec {\tilde
u}\cdot\nabla_{\vec x}S\right)}\left(-\nabla_{\vec
x}S\right)\times\vec B
\\
\left(-\nabla_{\vec x}S\right)\cdot \vec D\approx 0\\
 \left(-\nabla_{\vec x}S\right)\cdot \vec
B\approx 0.
\end{cases}
\end{equation}
So the vectors $(-\nabla_{\vec x}S)\,$, $\,\vec D$ and $\vec B$ form
together rightly orientated orthogonal system of vectors.

Next by \er{MaxVacFullPPNffGGhjhjhjhjkhkhljjkjnhhhgjgjg} and the
last two equations in
\er{MaxVacFullPPNmmmffffffiuiuhjuughbghhuiiujjnewjikhkjhjkkjhh} we
deduce:
\begin{multline}\label{MaxVacFullPPNffGGhjhjhjhjkhkhljjkjnhhhhnhkhk}
\vec E_1\approx\left(\frac{\gamma_0\kappa_0}{\left(1+\frac{1}{c}\vec
{\tilde u}\cdot\nabla_{\vec x}S\right)}\nabla_{\vec
x}S-\frac{1}{c}\vec {\tilde u}\right)\times\vec
B_1\quad\text{and}\quad \vec
H_1\approx-\left(\frac{\gamma_0\kappa_0}{\left(1+\frac{1}{c}\vec
{\tilde u}\cdot\nabla_{\vec x}S\right)}\nabla_{\vec
x}S-\frac{1}{c}\vec {\tilde u}\right)\times\vec D_1,
\end{multline}
and in particular,
\begin{equation}\label{MaxVacFullPPNffGGhjhjhjhjkhkhljjkjnhhhhnhkhkojj}
\left(\frac{\gamma_0\kappa_0}{\left(1+\frac{1}{c}\vec {\tilde
u}\cdot\nabla_{\vec x}S\right)}\nabla_{\vec x}S-\frac{1}{c}\vec
{\tilde u}\right)\cdot\vec E_1\approx 0\quad\text{and}\quad
\left(\frac{\gamma_0\kappa_0}{\left(1+\frac{1}{c}\vec {\tilde
u}\cdot\nabla_{\vec x}S\right)}\nabla_{\vec x}S-\frac{1}{c}\vec
{\tilde u}\right)\cdot\vec H_1\approx 0.
\end{equation}
Then, by \er{MaxVacFullPPNffGGhjhjhjhjkhkhljjkjnhhhhnhkhk} and by
the last two equations in
\er{MaxVacFullPPNmmmffffffiuiuhjuughbghhuiiujjnewjikhkjhjkkjhh},
using \er{ojhkkkhhhh} we deduce:
\begin{multline}\label{MaxVacFullPPNffGGhjhjhjhjkhkhljjkjnhhhhnhkhkhhjhjhj}
\vec
E_1\approx\frac{1}{\kappa_0}\left(\frac{\gamma_0\kappa_0}{\left(1+\frac{1}{c}\vec
{\tilde u}\cdot\nabla_{\vec x}S\right)}\nabla_{\vec
x}S-\frac{1}{c}\vec {\tilde u}\right)\times\left(\vec
H_1-\frac{1}{c\gamma_0}\,\vec {\tilde u}\times \vec
E_1\right)\quad\text{and}\\ \quad \vec
H_1\approx-\frac{1}{\gamma_0}\left(\frac{\gamma_0\kappa_0}{\left(1+\frac{1}{c}\vec
{\tilde u}\cdot\nabla_{\vec x}S\right)}\nabla_{\vec
x}S-\frac{1}{c}\vec {\tilde u}\right)\times\left(\vec
E_1+\frac{1}{c\kappa_0}\,\vec {\tilde u}\times \vec H_1\right)
\end{multline}
So, by \er{apfrm2} and
\er{MaxVacFullPPNffGGhjhjhjhjkhkhljjkjnhhhhnhkhkojj}, using
\er{ojhkkkhhhh}, we rewrite
\er{MaxVacFullPPNffGGhjhjhjhjkhkhljjkjnhhhhnhkhkhhjhjhj} as:
\begin{multline}\label{MaxVacFullPPNffGGhjhjhjhjkhkhljjkjnhhhhnhkhkhhjhjhjjhjgj}
\frac{1}{\left(1+\frac{1}{c}\vec {\tilde u}\cdot\nabla_{\vec
x}S\right)}\vec
E_1\approx\frac{1}{\kappa_0}\left(\frac{\gamma_0\kappa_0}{\left(1+\frac{1}{c}\vec
{\tilde u}\cdot\nabla_{\vec x}S\right)}\nabla_{\vec
x}S-\frac{1}{c}\vec {\tilde u}\right)\times\vec H_1\quad\text{and}\\
\quad\frac{1}{\left(1+\frac{1}{c}\vec {\tilde u}\cdot\nabla_{\vec
x}S\right)} \vec
H_1\approx-\frac{1}{\gamma_0}\left(\frac{\gamma_0\kappa_0}{\left(1+\frac{1}{c}\vec
{\tilde u}\cdot\nabla_{\vec x}S\right)}\nabla_{\vec
x}S-\frac{1}{c}\vec {\tilde u}\right)\times\vec E_1
\end{multline}
Then, using \er{ojhkkkhhhh} and
\er{MaxVacFullPPNmmmffffffiuiuhjuughbghhiuijghghghhjhjhhghyuyhhjhjihikjjjnew},
we finally rewrite
\er{MaxVacFullPPNffGGhjhjhjhjkhkhljjkjnhhhhnhkhkhhjhjhjjhjgj} as:
\begin{equation}\label{MaxVacFullPPNffGGhjhjhjhjkhkhljjkjnhhhhnhkhkhhjhjhjjhjgjlljkjuiy}
\vec E_1\approx
-\gamma_0\left(\frac{c}{c^2_0}\vec {\tilde u}-\nabla_{\vec
x}S\right)\times\vec H_1\quad\quad\text{and}\quad \quad \vec
H_1\approx\kappa_0
\left(\frac{c}{c^2_0}\vec {\tilde u}-\nabla_{\vec
x}S\right)\times\vec E_1\,,
\end{equation}
and in particular,
\begin{equation}\label{MaxVacFullPPNffGGhjhjhjhjkhkhljjkjnhhhhnhkhkhhjhjhjjhjgjlljkjuiyjjkljk}
\left(\frac{c}{c^2_0}\vec {\tilde u}-\nabla_{\vec
x}S\right)\cdot\vec E_1\approx 0\quad\quad\text{and}\quad \quad
\left(\frac{c}{c^2_0}\vec {\tilde u}-\nabla_{\vec x}S\right)\cdot
\vec H_1\approx 0.
\end{equation}
Thus, by inserting
\er{MaxVacFullPPNffGGhjhjhjhjkhkhljjkjnhhhhnhkhkhhjhjhjjhjgjlljkjuiy}
and
\er{MaxVacFullPPNffGGhjhjhjhjkhkhljjkjnhhhhnhkhkhhjhjhjjhjgjlljkjuiyjjkljk}
into
\er{MaxVacFullPPNmmmffffffiuiuhjuughbghhuiiujjnewjikhkjhjkkjhh}, and
using
\er{MaxVacFullPPNmmmffffffiuiuhjuughbghhiuijghghghhhfhhghghguygtjuuujjjkyuuykkhhk}
we finally deduce:
\begin{equation}\label{MaxVacFullPPNffGGhjhjhjhjkhkhljjkjnhhhhnhkhkhhjhjhjjhjgjlljkjuiyklk}
\begin{cases}
\vec H\approx\kappa_0
\vec h_0\times\vec E
\\
\vec E\approx
-\gamma_0\vec h_0\times\vec H\\
\vec h_0\cdot\vec E\approx 0\\
\vec h_0\cdot \vec H\approx 0.
\end{cases}
\end{equation}
So the vectors
$\vec h_0\,$, $\vec E$ and $\vec H$ form together rightly orientated
orthogonal system of vectors. We remind here again that $\vec h_0$
is the direction of the propagation of the ray.

\subsubsection{More delicate approximation of the wave
equation}\label{gjgggggggggggggggggggghkhk}
Again consider the
scalar wave equation of the form
\begin{equation}\label{MaxVacFullPPNmmmffffffiuiuhjuughbghhjjjoojouiuiu}
\left(\frac{\partial}{\partial
t}\left\{\frac{1}{c^2_0}\left(\frac{\partial U}{\partial t}+\vec
{\tilde u}\cdot\nabla_{\vec x}U\right)\right\}+\Div_{\vec x}
\left\{\frac{1}{c^2_0}\left(\frac{\partial U}{\partial t}+\vec
{\tilde u}\cdot\nabla_{\vec x} U\right)\vec {\tilde
u}\right\}\right)-\Delta_{\vec x}U=0.
\end{equation}
Above, we established $U(\vec x,t)$ only in the \underline{rough}
Geometric Optics approximation \er{slochangGGaaffgfgjhjggh11}. Since
we have \er{MaxVacFullPPNmmmffffffiuiuhjuughbghhuiiujjjjjjjjgjjg} in
the rough Geometric Optics approximation and $S(\vec x,t)$ satisfies
the following Eikonal equation
\begin{equation}\label{MaxVacFullPPNmmmffffffiuiuhjuughbghhiuijghghghhjhjhhghyuyjjjhhjhjffhjjhjoj}
\frac{1}{c^2_0}\left(\frac{\partial S}{\partial t}+\vec {\tilde
u}\cdot\nabla_\vec x S\right)^2=\left|\nabla_\vec x S\right|^2.
\end{equation} in the case of more precise
delicate Geometric Optics approximation, in order to get better
approximation for $U(\vec x,t)$ we can use the following alternative
consideration, slightly different than we did above: from now we
consider $U(\vec x,t)$ be of the form
\begin{equation}\label{MaxVacFullPPNmmmffffffiuiuhjuughbghhuiiujjjjjjjjgjjgiuiuui}
U(\vec x,t)=A(\vec x,t)e^{ik_0S(\vec x,t)}\,,
\end{equation}
similarly to
\er{MaxVacFullPPNmmmffffffiuiuhjuughbghhuiiujjjjjjjjgjjg}. However,
we consider that $S$ in
\er{MaxVacFullPPNmmmffffffiuiuhjuughbghhuiiujjjjjjjjgjjgiuiuui}
satisfied
\er{MaxVacFullPPNmmmffffffiuiuhjuughbghhiuijghghghhjhjhhghyuyjjjhhjhjffhjjhjoj}
\underline{precisely}, although we consider the amplitude $A$ to be
\underline{complex} rather than real! So we include the correction
term of $\kappa_0 S$ into the complex amplitude $A$! Thus either
\er{MaxVacFullPPNmmmffffffiuiuhjuughbghhiuijghghghhhfhhghghguygtjjjkkjjkkhh}
or
\er{MaxVacFullPPNmmmffffffiuiuhjuughbghhiuijghghghhhfhhghghguygtjjjkkjjkkhhklklkjk33hhjjjlj}
is not valid anymore, since it was derived under the assumption of
real amplitude $A$. However, either
\er{MaxVacFullPPNmmmffffffiuiuhjuughbghhiuijghghghhjhjhhghyuyjjjhhjhjff}
or
\er{MaxVacFullPPNmmmffffffiuiuhjuughbghhiuijghghghhhgjgfghjdfjgddg1133khhujuuuo99klklljljlj}
still holds. I.e. we have
\er{MaxVacFullPPNmmmffffffiuiuhjuughbghhiuijghghghhjhjhhghyuyjjjhhjhjffhjjhjoj},
which assumed to hold precisely. Moreover, by
\er{MaxVacFullPPNmmmffffffiuiuhjuughbghhiuijghghghhjhjhjjuiuiui}
together with
\er{MaxVacFullPPNmmmffffffiuiuhjuughbghhiuijghghghhjhjhhghyuyjjjhhjhjffhjjhjoj}
we have
\begin{multline}\label{MaxVacFullPPNmmmffffffiuiuhjuughbghhiuijghghghhjhjhjjuiuiuiuuiui}
\left(\frac{\partial}{\partial
t}\left\{\frac{1}{c^2_0}\left(\frac{\partial A}{\partial t}+\vec
{\tilde u}\cdot\nabla_{\vec x}A\right)\right\}+\Div_{\vec x}
\left\{\frac{1}{c^2_0}\left(\frac{\partial A}{\partial t}+\vec
{\tilde u}\cdot\nabla_{\vec x} A\right)\vec {\tilde
u}\right\}\right)-\Delta_{\vec
x}A\\
+i\kappa_0 A\left(\frac{\partial}{\partial
t}\left\{\frac{1}{c^2_0}\left(\frac{\partial S}{\partial t}+\vec
{\tilde u}\cdot\nabla_{\vec x}S\right)\right\}+\Div_{\vec x}
\left\{\frac{1}{c^2_0}\left(\frac{\partial S}{\partial t}+\vec
{\tilde u}\cdot\nabla_{\vec x} S\right)\vec {\tilde
u}\right\}-\Delta_{\vec
x}S\right)\\+2i\kappa_0\left(\frac{1}{c^2_0}\left(\frac{\partial
S}{\partial t}+\vec {\tilde u}\cdot\nabla_{\vec
x}S\right)\left(\frac{\partial A}{\partial t}+\vec {\tilde
u}\cdot\nabla_{\vec x}A\right)-\nabla_{\vec x}A\cdot\nabla_{\vec
x}S\right)=0.
\end{multline}
Then denoting, as before, the proper scalar field:
\begin{equation}\label{MaxVacFullPPNmmmffffffiuiuhjuughbghhiuijghghghhhgjgfghjdfjgddg1133khhujuuuo33}
\tilde\omega(\vec x,t):=\kappa_0\left(\frac{\partial S}{\partial
t}\left(\vec x,t\right)+\vec {\tilde u}\left(\vec
x,t\right)\cdot\nabla_{\vec x}S\left(\vec x,t\right)\right),
\end{equation}
the proper vector field:
\begin{equation}\label{MaxVacFullPPNmmmffffffiuiuhjuughbghhiuijghghghhhgjgfghjdfjgddg11335533}
\vec k(\vec x,t):=c_0(\vec x,t)\left(\frac{\partial S}{\partial
t}(\vec x,t)+\vec {\tilde u}(\vec x,t)\cdot\nabla_{\vec x}S(\vec
x,t)\right)^{-1}\nabla_{\vec x}S(\vec x,t),
\end{equation}
the speed-like vector field:
\begin{equation}\label{MaxVacFullPPNmmmffffffiuiuhjuughbghhiuijghghghhhgjgfghjdfjgddg113333}
\vec h(\vec x,t):=\vec {\tilde u}(\vec x,t)-c^2_0(\vec
x,t)\left(\frac{\partial S}{\partial t}(\vec x,t)+\vec {\tilde
u}(\vec x,t)\cdot\nabla_{\vec x}S(\vec x,t)\right)^{-1}\nabla_{\vec
x}S(\vec x,t)=\vec {\tilde u}(\vec x,t)-c_0(\vec x,t)\vec k(\vec
x,t),
\end{equation}
and the proper scalar field:
\begin{multline}\label{MaxVacFullPPNmmmffffffiuiuhjuughbghhiuijghghghhhfhhghghguygtjuuujjjkhjhjh11kklkloiouu3333}
G\left(\vec x,t\right):=\\
\left(\frac{c^2_0}{2\left(\frac{\partial S}{\partial t}+\vec {\tilde
u}\cdot\nabla_{\vec x}S\right)}\left(\Delta_{\vec
x}S-\frac{\partial}{\partial
t}\left\{\frac{1}{c^2_0}\left(\frac{\partial S}{\partial t}+\vec
{\tilde u}\cdot\nabla_{\vec x}S\right)\right\}-\Div_{\vec x}
\left\{\frac{1}{c^2_0}\left(\frac{\partial S}{\partial t}+\vec
{\tilde u}\cdot\nabla_{\vec x} S\right)\vec {\tilde
u}\right\}\right)\right)(\vec x,t)
\\=
\left(\frac{c^2_0}{2\tilde\omega}\left(\Div_{\vec
x}\left\{\frac{\tilde\omega}{c_0}\,\vec
k\right\}-\frac{\partial}{\partial
t}\left\{\frac{\tilde\omega}{c^2_0}\right\}-\Div_{\vec x}
\left\{\frac{\tilde\omega}{c^2_0}\,\vec {\tilde
u}\right\}\right)\right)(\vec x,t)\\=-
\left(\frac{c^2_0}{2\tilde\omega}\left(\frac{\partial}{\partial
t}\left\{\frac{\tilde\omega}{c^2_0}\right\}+\vec h\cdot\nabla_{\vec
x} \left\{\frac{\tilde\omega}{c^2_0}\,\right\}\right)\right)(\vec
x,t)-\frac{1}{2}\left(\Div_{\vec x}\vec h(\vec x,t)\right),
\end{multline}
we clearly rewrite
\er{MaxVacFullPPNmmmffffffiuiuhjuughbghhiuijghghghhjhjhjjuiuiuiuuiui}
as:
\begin{multline}\label{MaxVacFullPPNmmmffffffiuiuhjuughbghhiuijghghghhhfhhghghguygtjjjkkjjkkhhklklkjk33hhjjklik}
\frac{c^2_0}{2\tilde\omega}\left\{\frac{\partial}{\partial
t}\left\{\frac{1}{c^2_0}\left(\frac{\partial A}{\partial t}+\vec
{\tilde u}\cdot\nabla_{\vec x}A\right)\right\}+\Div_{\vec x}
\left\{\frac{1}{c^2_0}\left(\frac{\partial A}{\partial t}+\vec
{\tilde u}\cdot\nabla_{\vec x} A\right)\vec {\tilde
u}\right\}-\Delta_{\vec x}A\right\}\\+i\left(\frac{\partial
A}{\partial t}
+\vec h\cdot\nabla_{\vec x}A
-G\, A\right)=0
,
\end{multline}
and
\er{MaxVacFullPPNmmmffffffiuiuhjuughbghhiuijghghghhjhjhhghyuyjjjhhjhjffhjjhjoj},
as
\begin{equation}\label{MaxVacFullPPNmmmffffffiuiuhjuughbghhiuijghghghhhgjgfghjdfjgddg1133khhujuuuo99klklljljljhjhjhyf11}
\left|\vec k(\vec x,t)\right|^2=1\quad\quad\text{or
equivalently}\quad\quad \left|\vec h(\vec x,t)-\vec {\tilde u}(\vec
x,t)\right|^2=c^2_0(\vec x,t)\,.
\end{equation}
Moreover, as before, we obviously have
\begin{equation}\label{MaxVacFullPPNmmmffffffiuiuhjuughbghhiuijghghghhhgjgfghjdfjgddg1133khhujuuuo99klklljljljhjhjhyf35}
\frac{\tilde\omega^2}{c^2_0k_0^2}=\frac{1}{c^2_0}\left(\frac{\partial
S}{\partial t}+\tilde{\vec u}\cdot\nabla_{\vec
x}S\right)^2=\left|\nabla_{\vec x}S\right|^2,\quad\quad \left|\vec
k\right|^2=1\quad\quad\text{or equivalently}\quad\quad \left|\vec
h-\vec {\tilde u}\right|^2=c^2_0\,,
\end{equation}
\begin{equation}\label{MaxVacFullPPNmmmffffffiuiuhjuughbghhiuijghghghhhgjgfghjdfjgddg1133khhujuuuo33uuu35}
\frac{\partial S}{\partial t}+\vec {\tilde u}\cdot\nabla_{\vec
x}S=c_0\vec k\cdot\nabla_{\vec x}S\quad\quad \text{or
equivalently}\quad\quad\quad\quad \frac{\partial S}{\partial t}+\vec
h\cdot\nabla_{\vec x}S=0\,,
\end{equation}
\begin{equation}\label{MaxVacFullPPNmmmffffffiuiuhjuughbghhiuijghghghhhgjgfghjdfjgddg1133khhujuuuo99klklljljljhjhjhyfjkkjkhkiili35}
\nabla_{\vec x}S\,=\,\left(\vec k(\vec x,t)\cdot\nabla_{\vec
x}S\right)\,\vec k\,,
\end{equation}
\begin{equation}\label{MaxVacFullPPNmmmffffffiuiuhjuughbghhiuijghghghhhgjgfghjdfjgddg11335533jkjjkujjj35}
\frac{c_0}{\tilde\omega}\vec k=\frac{1}{k_0}\left|\nabla_{\vec
x}S\right|^{-2}\nabla_{\vec
x}S=\frac{c^2_0k_0}{\tilde\omega^2}\nabla_{\vec x}S,
\end{equation}
\begin{equation}\label{MaxVacFullPPNmmmffffffiuiuhjuughbghhiuijghghghhhgjgfghjdfjgddg1133khhujuuuo33uuuiujkjk35}
\frac{1}{k_0}\,\nabla_{\vec x}\tilde\omega=\frac{\partial}{\partial
t}\left\{\nabla_{\vec x}S\right\}+\nabla^2_{\vec x}S\cdot\vec
{\tilde u}+\left\{d_{\vec x}\vec {\tilde
u}\right\}^T\cdot\nabla_{\vec x}S\,,
\end{equation}
and
\begin{equation}\label{MaxVacFullPPNmmmffffffiuiuhjuughbghhiuijghghghhhgjgfghjdfjgddg1133khhujuuuo33uuuiujkjkuyuyiioiiuuiiiiu35}
\frac{1}{\tilde\omega}\left(\frac{\partial \tilde\omega}{\partial
t}+\vec h\cdot\nabla_{\vec
x}\tilde\omega\right)=\frac{1}{c_0}\left(\frac{\partial
c_0}{\partial t}+\vec {\tilde u}\cdot\nabla_{\vec
x}c_0\right)-\frac{1}{2}\left(\vec k\cdot\left(\left\{d_{\vec x}\vec
{\tilde u}\right\}^T+d_{\vec x}\vec {\tilde
u}\right)\right)\cdot\vec k\,.
\end{equation}

Furthermore, from now assume the case of \underline{delicate}
Geometric Optics approximation. Remind that the delicate Geometric
Optics approximation means the following: assume that the changes in
time of $c_0$, $\vec {\tilde u}$, $A$ and $S$ become essential after
certain interval of time $T_e$ and the changes in space of $c_0$,
$A$ and $S$ become essential in the spatial landscape $L_e$. Then we
assume that
\begin{equation}\label{slochangGGaaffgfgjhjggh99}
k^2_0c^2_0T^2_e\,\gg\, 1\quad\text{and}\quad k^2_0L^2_e\,\gg\, 1.
\end{equation}
Moreover, as before, we assume that the order of $c_0$ is less or
equal to the order of $c$. Then, estimation
\er{slochangGGaaffgfgjhjggh99} implies
\begin{multline}\label{MaxVacFullPPNmmmffffffhhtygghGGGGyuhggghghghffgiooiiojhjjhhjuiuiioiooikjiij77}
\frac{\left| d_{\vec x}\vec {\tilde u}+\left\{d_{\vec x}\vec {\tilde
u}\right\}^T\right|^2}{
c^2_0}+\frac{\left(| \nabla_{\vec
x}A|+\frac{1}{c_0}\left|\frac{\partial A}{\partial t}+\vec {\tilde
u}\cdot\nabla_{\vec x}A\right|\right)^2}{|A|^2}+\frac{\left(|
\nabla_{\vec x}c_0|+\frac{1}{c_0}\left|\frac{\partial c_0}{\partial
t}+\vec {\tilde u}\cdot\nabla_{\vec
x}c_0\right|\right)^2}{|c_0|^2}\\+\frac{ \left| \nabla^2_{\vec
x}S\right|^2+ \frac{1}{c^2_0}\left|\nabla_{\vec
x}\left(\frac{\partial S}{\partial t}+\vec {\tilde
u}\cdot\nabla_{\vec x}S\right)\right|^2 }{ \left| \nabla_{\vec
x}S\right|^2} \ll k^2_0 \quad\quad\text{and}\quad\quad
\frac{\left|\frac{\partial S}{\partial t}+\vec {\tilde
u}\cdot\nabla_{\vec x}S-c\right|^2}{c^2}\ll 1\,,
\end{multline}
so that we  have
\begin{multline}\label{MaxVacFullPPNmmmffffffhhtygghGGGGyuhggghghghffgiooiiojhjjhhjuiuiioiooikjiij}
\frac{\left| d_{\vec x}\vec {\tilde u}+\left\{d_{\vec x}\vec {\tilde
u}\right\}^T\right|^2}{
c^2_0}+\frac{\left(| \nabla_{\vec
x}A|+\frac{1}{c_0}\left|\frac{\partial A}{\partial t}+\vec {\tilde
u}\cdot\nabla_{\vec x}A\right|\right)^2}{|A|^2}+\frac{\left(|
\nabla_{\vec x}c_0|+\frac{1}{c_0}\left|\frac{\partial c_0}{\partial
t}+\vec {\tilde u}\cdot\nabla_{\vec
x}c_0\right|\right)^2}{|c_0|^2}\\+\frac{ \left| \nabla^2_{\vec
x}S\right|^2+ \frac{1}{c^2_0}\left|\nabla_{\vec
x}\left(\frac{\partial S}{\partial t}+\vec {\tilde
u}\cdot\nabla_{\vec x}S\right)\right|^2 }{ \left| \nabla_{\vec
x}S\right|^2}  \ll
\frac{\tilde\omega^2}{c^2}\leq\frac{\tilde\omega^2}{c^2_0}\,.
\end{multline}
Moreover, by
\er{MaxVacFullPPNmmmffffffhhtygghGGGGyuhggghghghffgiooiiojhjjhhjuiuiioiooikjiij}
together with
\er{MaxVacFullPPNmmmffffffiuiuhjuughbghhiuijghghghhhgjgfghjdfjgddg1133khhujuuuo33uuuiujkjkuyuyiioiiuuiiiiu35}
we also have
\begin{equation}\label{MaxVacFullPPNmmmffffffhhtygghGGGGyuhggghghghffgiooiiojhjjhhjuiuiioiooiuyuhhjhj}
\frac{|\Div_{\vec x}\vec h|^2}{c^2_0}+\frac{\left(| d_{\vec x}\vec
k|+\frac{1}{c_0}\left|\frac{\partial \vec k}{\partial t}+d_{\vec
x}\vec k\cdot\vec {\tilde u}-d_{\vec x}\vec {\tilde u}\cdot\vec
k\right|\right)^2}{|\vec k|^2}+\frac{\left(| \nabla_{\vec
x}\tilde\omega|+\frac{1}{c_0}\left|\frac{\partial
\tilde\omega}{\partial t}+\vec {\tilde u}\cdot\nabla_{\vec
x}\tilde\omega\right|\right)^2}{|\tilde\omega|^2}\ll
\frac{\tilde\omega^2}{c^2_0}\,.
\end{equation}
Therefore, denoting
\begin{equation}\label{MaxVaijkjjk1}
H:=\frac{\partial A}{\partial t}
+\vec h\cdot\nabla_{\vec x}A
-G\,A=\frac{\partial A}{\partial t}+\vec {\tilde u}\cdot\nabla_{\vec
x}A-c_0\vec k\cdot\nabla_{\vec x}A-G\,A
\end{equation}
we rewrite
\er{MaxVacFullPPNmmmffffffiuiuhjuughbghhiuijghghghhhfhhghghguygtjjjkkjjkkhhklklkjk33hhjjklik}
as:
\begin{multline}\label{MaxVacFullPPNmmmffffffiuiuhjuughbghhiuijghghghhhfhhghghguygtjjjkkjjkkhhklklkjk33hhjjklikjjkjkjjjkoi1}
\frac{c^2_0}{2\tilde\omega}\left\{\frac{\partial }{\partial
t}\left(\frac{1}{c^2_0}H\right)+\Div_{\vec x}
\left(\frac{1}{c^2_0}H\vec {\tilde u}\right)\right\}
\\
+\frac{c^2_0}{2\tilde\omega}\left\{\frac{\partial}{\partial
t}\left\{\frac{1}{c^2_0}\left(c_0\vec k\cdot\nabla_{\vec
x}A+G\,A\right)\right\}+\Div_{\vec x}
\left\{\frac{1}{c^2_0}\left(c_0\vec k\cdot\nabla_{\vec
x}A+G\,A\right)\vec {\tilde u}\right\}\right\}
-\frac{c^2_0}{2\tilde\omega}\left(\Delta_{\vec x}A\right)=-iH\,
.
\end{multline}
However, by
\er{MaxVacFullPPNmmmffffffhhtygghGGGGyuhggghghghffgiooiiojhjjhhjuiuiioiooikjiij}
and
\er{MaxVacFullPPNmmmffffffhhtygghGGGGyuhggghghghffgiooiiojhjjhhjuiuiioiooiuyuhhjhj}
we have
\begin{equation}\label{MaxVacFullPPNmmmffffffhhtygghGGGGyuhggghghghffgiooiiojhjjhhjuiuiioiooiuyuhhjhjjhjh13}
\frac{\left(| \nabla_{\vec x}H|+\frac{1}{c_0}\left|\frac{\partial
H}{\partial t}+\vec {\tilde u}\cdot\nabla_{\vec
x}H\right|\right)^2}{|H|^2}\ll \frac{\tilde\omega^2}{c^2_0}\,.
\end{equation}
In particular, by
\er{MaxVacFullPPNmmmffffffiuiuhjuughbghhiuijghghghhhfhhghghguygtjjjkkjjkkhhklklkjk33hhjjklikjjkjkjjjkoi1},
\er{MaxVacFullPPNmmmffffffhhtygghGGGGyuhggghghghffgiooiiojhjjhhjuiuiioiooikjiij},
\er{MaxVacFullPPNmmmffffffhhtygghGGGGyuhggghghghffgiooiiojhjjhhjuiuiioiooiuyuhhjhj}
and
\er{MaxVacFullPPNmmmffffffhhtygghGGGGyuhggghghghffgiooiiojhjjhhjuiuiioiooiuyuhhjhjjhjh13}
we deduce
\begin{equation}\label{MaxVacFullPPNmmmffffffhhtygghGGGGyuhggghghghffgiooiiojhjjhhjuiuiioiooiuyuhhjhjjhjh11}
H^2\ll
\left|c_0\vec k\cdot\nabla_{\vec x}A
+G\,A\right|^2
\end{equation}
Thus, by
\er{MaxVacFullPPNmmmffffffhhtygghGGGGyuhggghghghffgiooiiojhjjhhjuiuiioiooikjiij},
\er{MaxVacFullPPNmmmffffffhhtygghGGGGyuhggghghghffgiooiiojhjjhhjuiuiioiooiuyuhhjhj},
\er{MaxVacFullPPNmmmffffffhhtygghGGGGyuhggghghghffgiooiiojhjjhhjuiuiioiooiuyuhhjhjjhjh13}
and
\er{MaxVacFullPPNmmmffffffhhtygghGGGGyuhggghghghffgiooiiojhjjhhjuiuiioiooiuyuhhjhjjhjh11}
we approximate
\er{MaxVacFullPPNmmmffffffiuiuhjuughbghhiuijghghghhhfhhghghguygtjjjkkjjkkhhklklkjk33hhjjklikjjkjkjjjkoi1}
as:
\begin{multline}\label{MaxVacFullPPNmmmffffffiuiuhjuughbghhiuijghghghhhfhhghghguygtjjjkkjjkkhhklklkjk33hhjjklikjjkjkjjjkoi;;lll177}
\frac{c^2_0}{2\tilde\omega}\left\{\frac{\partial}{\partial
t}\left\{\frac{1}{c^2_0}\left(c_0\vec k\cdot\nabla_{\vec
x}A+G\,A\right)\right\}+\Div_{\vec x}
\left\{\frac{1}{c^2_0}\left(c_0\vec k\cdot\nabla_{\vec
x}A+G\,A\right)\vec {\tilde u}\right\}\right\}
-\frac{c^2_0}{2\tilde\omega}\left(\Delta_{\vec x}A\right)\approx-iH=\\
-i\left\{\frac{\partial A}{\partial t}
+\vec h\cdot\nabla_{\vec x}A
+\frac{1}{2}\,\left(\Div_{\vec x}\vec h \right)A\right\}-
iA\left(\frac{c_0}{\sqrt{\tilde\omega}}\left(\frac{\partial}{\partial
t}\left\{\frac{\sqrt{\tilde\omega}}{c_0}\right\}+\vec
h\cdot\nabla_{\vec x}
\left\{\frac{\sqrt{\tilde\omega}}{c_0}\,\right\}\right)\right).
\end{multline}
We rewrite it as
\begin{multline}\label{MaxVacFullPPNmmmffffffiuiuhjuughbghhiuijghghghhhfhhghghguygtjjjkkjjkkhhklklkjk33hhjjklikjjkjkjjjkoi;;lll1}
\frac{c^2_0}{2\tilde\omega}\nabla_{\vec
x}A\cdot\left\{\frac{\partial}{\partial t}\left\{\frac{1}{c_0}\vec
k\right\}+\Div_{\vec x} \left\{\frac{1}{c_0}\vec k\otimes\vec
{\tilde u}\right\}-d_{\vec x}\vec {\tilde
u}\cdot\left(\frac{1}{c_0}\vec
k\right)\right\}+\frac{1}{2\tilde\omega}c_0\vec k\cdot\nabla_{\vec
x}\left(\frac{\partial A}{\partial t}+\vec {\tilde
u}\cdot\nabla_{\vec x}A\right)
\\+\frac{1}{2\tilde\omega}G\left(\frac{\partial A}{\partial t}+\vec
{\tilde u}\cdot\nabla_{\vec x}A\right)
+\frac{c^2_0}{2\tilde\omega}A\left(\frac{\partial}{\partial
t}\left\{\frac{1}{c^2_0}G\right\}+\Div_{\vec x}
\left\{\frac{1}{c^2_0}G\vec {\tilde u}\right\}\right)
-\frac{c^2_0}{2\tilde\omega}\left(\Delta_{\vec x}A\right)\approx-iH=\\
-i\left\{\frac{\partial A}{\partial t}
+\vec h\cdot\nabla_{\vec x}A
+\frac{1}{2}\,\left(\Div_{\vec x}\vec h \right)A\right\}-
iA\left(\frac{c_0}{\sqrt{\tilde\omega}}\left(\frac{\partial}{\partial
t}\left\{\frac{\sqrt{\tilde\omega}}{c_0}\right\}+\vec
h\cdot\nabla_{\vec x}
\left\{\frac{\sqrt{\tilde\omega}}{c_0}\,\right\}\right)\right).
\end{multline}
Therefore,
\begin{multline}\label{MaxVacFullPPNmmmffffffiuiuhjuughbghhiuijghghghhhfhhghghguygtjjjkkjjkkhhklklkjk33hhjjklikjjkjkjjjkoi;;lll1jkkjkj1}
\frac{c^2_0}{2\tilde\omega}\nabla_{\vec
x}A\cdot\left\{\frac{\partial}{\partial t}\left\{\frac{1}{c_0}\vec
k\right\}+\Div_{\vec x} \left\{\frac{1}{c_0}\vec k\otimes\vec
{\tilde u}\right\}-d_{\vec x}\vec {\tilde
u}\cdot\left(\frac{1}{c_0}\vec
k\right)\right\}\\+\frac{1}{2\tilde\omega}c_0\vec k\cdot\nabla_{\vec
x}H +\frac{1}{2\tilde\omega}G\,H +\frac{1}{2\tilde\omega}c_0\vec
k\cdot\nabla_{\vec x}\left(c_0\vec k\cdot\nabla_{\vec
x}A+G\,A\right)+\frac{1}{2\tilde\omega}G\,\left(c_0\vec
k\cdot\nabla_{\vec x}A+G\,A\right)\\
+\frac{c^2_0}{2\tilde\omega}A\left(\frac{\partial}{\partial
t}\left\{\frac{1}{c^2_0}G\right\}+\Div_{\vec x}
\left\{\frac{1}{c^2_0}G\vec {\tilde u}\right\}\right)
-\frac{c^2_0}{2\tilde\omega}\left(\Delta_{\vec x}A\right)\approx-iH=\\
-i\left\{\frac{\partial A}{\partial t}
+\vec h\cdot\nabla_{\vec x}A
+\frac{1}{2}\,\left(\Div_{\vec x}\vec h \right)A\right\}-
iA\left(\frac{c_0}{\sqrt{\tilde\omega}}\left(\frac{\partial}{\partial
t}\left\{\frac{\sqrt{\tilde\omega}}{c_0}\right\}+\vec
h\cdot\nabla_{\vec x}
\left\{\frac{\sqrt{\tilde\omega}}{c_0}\,\right\}\right)\right).
\end{multline}
Thus, by
\er{MaxVacFullPPNmmmffffffhhtygghGGGGyuhggghghghffgiooiiojhjjhhjuiuiioiooiuyuhhjhj},
\er{MaxVacFullPPNmmmffffffhhtygghGGGGyuhggghghghffgiooiiojhjjhhjuiuiioiooiuyuhhjhjjhjh13}
and
\er{MaxVacFullPPNmmmffffffhhtygghGGGGyuhggghghghffgiooiiojhjjhhjuiuiioiooiuyuhhjhjjhjh11}
we further approximate
\er{MaxVacFullPPNmmmffffffiuiuhjuughbghhiuijghghghhhfhhghghguygtjjjkkjjkkhhklklkjk33hhjjklikjjkjkjjjkoi;;lll1jkkjkj1}
as
\begin{multline}\label{MaxVacFullPPNmmmffffffiuiuhjuughbghhiuijghghghhhfhhghghguygtjjjkkjjkkhhklklkjk33hhjjklikjjkjkjjjkoi;;lll1jkkjkj}
\frac{c^2_0}{2\tilde\omega}\nabla_{\vec
x}A\cdot\left\{\frac{\partial}{\partial t}\left\{\frac{1}{c_0}\vec
k\right\}+\Div_{\vec x} \left\{\frac{1}{c_0}\vec k\otimes\vec
{\tilde u}\right\}-d_{\vec x}\vec {\tilde
u}\cdot\left(\frac{1}{c_0}\vec k\right)\right\}\\
+\frac{1}{2\tilde\omega}c_0\vec k\cdot\nabla_{\vec x}\left(c_0\vec
k\cdot\nabla_{\vec
x}A+G\,A\right)+\frac{1}{2\tilde\omega}G\,\left(c_0\vec
k\cdot\nabla_{\vec x}A+G\,A\right)\\
+\frac{c^2_0}{2\tilde\omega}A\left(\frac{\partial}{\partial
t}\left\{\frac{1}{c^2_0}G\right\}+\Div_{\vec x}
\left\{\frac{1}{c^2_0}G\vec {\tilde u}\right\}\right)
-\frac{c^2_0}{2\tilde\omega}\left(\Delta_{\vec x}A\right)\approx\\
-i\left\{\frac{\partial A}{\partial t}
+\vec h\cdot\nabla_{\vec x}A
+\frac{1}{2}\,\left(\Div_{\vec x}\vec h \right)A\right\}-
iA\left(\frac{c_0}{\sqrt{\tilde\omega}}\left(\frac{\partial}{\partial
t}\left\{\frac{\sqrt{\tilde\omega}}{c_0}\right\}+\vec
h\cdot\nabla_{\vec x}
\left\{\frac{\sqrt{\tilde\omega}}{c_0}\,\right\}\right)\right).
\end{multline}
We rewrite the last equation as:
\begin{multline}\label{MaxVacFullPPNmmmffffffiuiuhjuughbghhiuijghghghhhfhhghghguygtjjjkkjjkkhhklklkjk33hhjjklikjjkjkjjjkoi;;lll1jkkjkjjljjjk77}
\frac{c^2_0}{2\tilde\omega}\nabla_{\vec
x}A\cdot\left\{\frac{\partial}{\partial t}\left\{\frac{1}{c_0}\vec
k\right\}+\Div_{\vec x} \left\{\frac{1}{c_0}\vec k\otimes\vec
{\tilde u}\right\}-d_{\vec x}\vec {\tilde
u}\cdot\left(\frac{1}{c_0}\vec k\right)\right\}\\
+\frac{c^2_0}{2\tilde\omega}\left(\frac{1}{c_0}\vec
k\right)\cdot\nabla_{\vec x}\left(c_0\vec k\cdot\nabla_{\vec
x}A\right)+\frac{1}{\tilde\omega}G\,\left(c_0\vec
k\cdot\nabla_{\vec x}A\right)\\
+\frac{c^2_0}{2\tilde\omega}A\left(\frac{\partial}{\partial
t}\left\{\frac{1}{c^2_0}G\right\}+\Div_{\vec x}
\left\{\frac{1}{c^2_0}G\vec {\tilde
u}\right\}+\frac{1}{c^2_0}(c_0\vec k)\cdot\nabla_{\vec
x}G+\frac{1}{c^2_0}G^2\right)
-\frac{c^2_0}{2\tilde\omega}\left(\Delta_{\vec x}A\right)\approx\\
-i\left\{\frac{\partial A}{\partial t}
+\vec h\cdot\nabla_{\vec x}A
+\frac{1}{2}\,\left(\Div_{\vec x}\vec h \right)A\right\}-
iA\left(\frac{c_0}{\sqrt{\tilde\omega}}\left(\frac{\partial}{\partial
t}\left\{\frac{\sqrt{\tilde\omega}}{c_0}\right\}+\vec
h\cdot\nabla_{\vec x}
\left\{\frac{\sqrt{\tilde\omega}}{c_0}\,\right\}\right)\right),
\end{multline}
and further as
\begin{multline}\label{MaxVacFullPPNmmmffffffiuiuhjuughbghhiuijghghghhhfhhghghguygtjjjkkjjkkhhklklkjk33hhjjklikjjkjkjjjkoi;;lll1jkkjkjjljjjk}
\frac{c^2_0}{2\tilde\omega}\nabla_{\vec
x}A\cdot\left\{\frac{\partial}{\partial t}\left\{\frac{1}{c_0}\vec
k\right\}+\Div_{\vec x} \left\{\frac{1}{c_0}\vec k\otimes\vec
{\tilde u}\right\}-d_{\vec x}\vec {\tilde
u}\cdot\left(\frac{1}{c_0}\vec k\right)\right\}\\
+\frac{c^2_0}{2\tilde\omega}\left(\frac{1}{c_0}\vec
k\cdot\nabla_{\vec x}c_0\right)\left(\vec k\cdot\nabla_{\vec
x}A\right)+\frac{c^2_0}{2\tilde\omega}\nabla_{\vec x}\left(\vec
k\cdot\nabla_{\vec x}A\right)\cdot\vec
k+\frac{1}{\tilde\omega}G\,\left(c_0\vec
k\cdot\nabla_{\vec x}A\right)\\
+\frac{c^2_0}{2\tilde\omega}A\left(\frac{\partial}{\partial
t}\left\{\frac{1}{c^2_0}G\right\}+\Div_{\vec x}
\left\{\frac{1}{c^2_0}G\vec {\tilde
u}\right\}+\frac{1}{c^2_0}(c_0\vec k)\cdot\nabla_{\vec
x}G+\frac{1}{c^2_0}G^2\right)
-\frac{c^2_0}{2\tilde\omega}\left(\Delta_{\vec x}A\right)\approx\\
-i\left\{\frac{\partial A}{\partial t}
+\vec h\cdot\nabla_{\vec x}A
+\frac{1}{2}\,\left(\Div_{\vec x}\vec h \right)A\right\}-
iA\left(\frac{c_0}{\sqrt{\tilde\omega}}\left(\frac{\partial}{\partial
t}\left\{\frac{\sqrt{\tilde\omega}}{c_0}\right\}+\vec
h\cdot\nabla_{\vec x}
\left\{\frac{\sqrt{\tilde\omega}}{c_0}\,\right\}\right)\right),
\end{multline}
Therefore, we infer
\begin{multline}\label{MaxVacFullPPNmmmffffffiuiuhjuughbghhiuijghghghhhfhhghghguygtjjjkkjjkkhhklklkjk33hhjjklikjjkjkjjjkoi;;lll1jkkjkjjljjjkjjkjk}
\frac{c^2_0}{2\tilde\omega}\nabla_{\vec
x}A\cdot\left\{\frac{\partial}{\partial t}\left\{\frac{1}{c_0}\vec
k\right\}+\Div_{\vec x} \left\{\frac{1}{c_0}\vec k\otimes\vec
{\tilde u}\right\}-d_{\vec x}\vec {\tilde
u}\cdot\left(\frac{1}{c_0}\vec k\right)\right\}\\
+\frac{c^2_0}{2\tilde\omega}\left(\frac{1}{c_0}\vec
k\cdot\nabla_{\vec x}c_0\right)\left(\vec k\cdot\nabla_{\vec
x}A\right)-\frac{c^2_0}{2\tilde\omega}\left(\vec k\cdot\nabla_{\vec
x}A\right)\left(\Div_{\vec x}\vec
k\right)+\frac{1}{\tilde\omega}G\,\left(c_0\vec
k\cdot\nabla_{\vec x}A\right)\\
+\frac{c^2_0}{2\tilde\omega}A\left(\frac{\partial}{\partial
t}\left\{\frac{1}{c^2_0}G\right\}+\Div_{\vec x}
\left\{\frac{1}{c^2_0}G\vec {\tilde
u}\right\}+\frac{1}{c^2_0}(c_0\vec k)\cdot\nabla_{\vec
x}G+\frac{1}{c^2_0}G^2\right)
-\frac{c^2_0}{2\tilde\omega}\left(\Div_{\vec x}\left\{\nabla_{\vec
x}A-\left(\vec
k\cdot\nabla_{\vec x}A\right)\vec k\right\}\right)\approx\\
-i\left\{\frac{\partial A}{\partial t}
+\vec h\cdot\nabla_{\vec x}A
+\frac{1}{2}\,\left(\Div_{\vec x}\vec h \right)A\right\}-
iA\left(\frac{c_0}{\sqrt{\tilde\omega}}\left(\frac{\partial}{\partial
t}\left\{\frac{\sqrt{\tilde\omega}}{c_0}\right\}+\vec
h\cdot\nabla_{\vec x}
\left\{\frac{\sqrt{\tilde\omega}}{c_0}\,\right\}\right)\right),
\end{multline}
So,
\begin{multline}\label{MaxVacFullPPNmmmffffffiuiuhjuughbghhiuijghghghhhfhhghghguygtjjjkkjjkkhhklklkjk33hhjjklikjjkjkjjjkoi;;lll1jkkjkjjljjjkjjkjkjjjkk}
\frac{c^2_0}{2\tilde\omega}\nabla_{\vec
x}A\cdot\left\{\frac{\partial}{\partial t}\left\{\frac{1}{c_0}\vec
k\right\}+\Div_{\vec x} \left\{\frac{1}{c_0}\vec k\otimes\vec
{\tilde u}\right\}-d_{\vec x}\vec {\tilde
u}\cdot\left(\frac{1}{c_0}\vec k\right)-\left(\Div_{\vec x}\left\{\frac{1}{c_0}\vec k\right\}\right)(c_0\vec k)+2G\left(\frac{1}{c_0}\vec k\right)\right\}\\
+\frac{c^2_0}{2\tilde\omega}A\left(\frac{\partial}{\partial
t}\left\{\frac{1}{c^2_0}G\right\}+\Div_{\vec x}
\left\{\frac{1}{c^2_0}G\vec {\tilde
u}\right\}+\frac{1}{c^2_0}(c_0\vec k)\cdot\nabla_{\vec
x}G+\frac{1}{c^2_0}G^2\right)
-\frac{c^2_0}{2\tilde\omega}\left(\Div_{\vec x}\left\{\nabla_{\vec
x}A-\left(\vec
k\cdot\nabla_{\vec x}A\right)\vec k\right\}\right)\approx\\
-i\left\{\frac{\partial A}{\partial t}
+\vec h\cdot\nabla_{\vec x}A
+\frac{1}{2}\,\left(\Div_{\vec x}\vec h \right)A\right\}-
iA\left(\frac{c_0}{\sqrt{\tilde\omega}}\left(\frac{\partial}{\partial
t}\left\{\frac{\sqrt{\tilde\omega}}{c_0}\right\}+\vec
h\cdot\nabla_{\vec x}
\left\{\frac{\sqrt{\tilde\omega}}{c_0}\,\right\}\right)\right),
\end{multline}
On the other hand by
\er{MaxVacFullPPNmmmffffffiuiuhjuughbghhiuijghghghhhgjgfghjdfjgddg11335533jkjjkujjj35}
we have
\begin{equation}\label{MaxVacFullPPNmmmffffffiuiuhjuughbghhiuijghghghhhgjgfghjdfjgddg11335533jkjjkujjj35hhh}
\frac{\tilde\omega}{k_0c_0}\vec k=\nabla_{\vec x}S,
\end{equation}
Thus,
\begin{multline}\label{MaxVacFullPPNmmmffffffiuiuhjuughbghhiuijghghghhhfhhghghguygtjjjkkjjkkhhklklkjk33hhjjklikjjkjkjjjkoi;;lll1jkkjkjjljjjkjjkjkjjjkkjjjuiuuiuiuuikjkj77}
\frac{\tilde\omega}{k_0}\left\{\frac{\partial}{\partial
t}\left\{\frac{1}{c_0}\vec k\right\}+\Div_{\vec x}
\left\{\frac{1}{c_0}\vec k\otimes\vec {\tilde u}\right\}-d_{\vec
x}\vec {\tilde u}\cdot\left(\frac{1}{c_0}\vec
k\right)-\left(\Div_{\vec x}\left\{\frac{1}{c_0}\vec
k\right\}\right)(c_0\vec k)+2G\left(\frac{1}{c_0}\vec
k\right)\right\}\\= \left\{\frac{\partial}{\partial
t}\left\{\nabla_{\vec x}S\right\}+\Div_{\vec x} \left\{\nabla_{\vec
x}S\otimes\vec {\tilde u}\right\}-d_{\vec x}\vec {\tilde
u}\cdot\left(\nabla_{\vec x}S\right)-\left(\Div_{\vec
x}\left\{\nabla_{\vec x}S\right\}\right)(c_0\vec
k)+2G\left(\nabla_{\vec x}S\right)\right\}\\-
\frac{1}{c_0}\left\{\frac{\partial}{\partial
t}\left\{\frac{\tilde\omega}{k_0}\right\}+\vec {\tilde
u}\cdot\nabla_{\vec x}
\left\{\frac{\tilde\omega}{k_0}\right\}-(c_0\vec k)\cdot\nabla_{\vec
x} \left\{\frac{\tilde\omega}{k_0}\right\}\right\}\vec
k=\\
\nabla_{\vec x}\left(\frac{\partial S}{\partial t}+\vec {\tilde
u}\cdot\nabla_{\vec x}S\right)+\left(\Div_{\vec x} \left\{\vec
{\tilde u}\right\}\right)\nabla_{\vec x}S-\left(d_{\vec x}\vec
{\tilde u}+\left\{d_{\vec x}\vec {\tilde
u}\right\}^T\right)\cdot\left(\nabla_{\vec
x}S\right)-\left(\Div_{\vec x}\left\{\nabla_{\vec
x}S\right\}\right)(c_0\vec k)+2G\left(\nabla_{\vec x}S\right)\\-
\frac{1}{c_0}\left\{\frac{\partial}{\partial
t}\left\{\frac{\tilde\omega}{k_0}\right\}+\vec h\cdot\nabla_{\vec x}
\left\{\frac{\tilde\omega}{k_0}\right\}\right\}\vec k.
\end{multline}
I.e.
\begin{multline}\label{MaxVacFullPPNmmmffffffiuiuhjuughbghhiuijghghghhhfhhghghguygtjjjkkjjkkhhklklkjk33hhjjklikjjkjkjjjkoi;;lll1jkkjkjjljjjkjjkjkjjjkkjjjuiuuiuiuuikjkj}
\tilde\omega\left\{\frac{\partial}{\partial
t}\left\{\frac{1}{c_0}\vec k\right\}+\Div_{\vec x}
\left\{\frac{1}{c_0}\vec k\otimes\vec {\tilde u}\right\}-d_{\vec
x}\vec {\tilde u}\cdot\left(\frac{1}{c_0}\vec
k\right)-\left(\Div_{\vec x}\left\{\frac{1}{c_0}\vec
k\right\}\right)(c_0\vec k)+2G\left(\frac{1}{c_0}\vec
k\right)\right\}=\\
\nabla_{\vec x}\tilde\omega+\left(\Div_{\vec x} \left\{\vec {\tilde
u}\right\}\right)\left(\frac{\tilde\omega}{c_0}\vec
k\right)-\left(d_{\vec x}\vec {\tilde u}+\left\{d_{\vec x}\vec
{\tilde u}\right\}^T\right)\cdot\left(\frac{\tilde\omega}{c_0}\vec
k\right)-\left(\Div_{\vec x}\left\{\frac{\tilde\omega}{c_0}\vec
k\right\}\right)(c_0\vec k)+2G\left(\frac{\tilde\omega}{c_0}\vec
k\right)\\- \frac{1}{c_0}\left\{\frac{\partial\tilde\omega}{\partial
t}+\vec h\cdot\nabla_{\vec x} \tilde\omega\right\}\vec
k=\nabla_{\vec
x}\tilde\omega-\left(\vec k\cdot\nabla_{\vec x}\tilde\omega\right)\vec k\\
+\tilde\omega\left(\Div_{\vec x} \left\{\vec {\tilde
u}\right\}\right)\left(\frac{1}{c_0}\vec
k\right)-\tilde\omega\left(d_{\vec x}\vec {\tilde u}+\left\{d_{\vec
x}\vec {\tilde u}\right\}^T\right)\cdot\left(\frac{1}{c_0}\vec
k\right)-\tilde\omega\left(c^2_0\,\Div_{\vec
x}\left\{\frac{1}{c_0}\vec k\right\}\right)\left(\frac{1}{c_0}\vec
k\right)+2G\left(\frac{\tilde\omega}{c_0}\vec k\right)\\-
\frac{1}{c_0}\left\{\frac{\partial\tilde\omega}{\partial t}+\vec
h\cdot\nabla_{\vec x} \tilde\omega\right\}\vec k =\nabla_{\vec
x}\tilde\omega-\left(\vec k\cdot\nabla_{\vec x}\tilde\omega\right)\vec k\\
+\tilde\omega\left(\Div_{\vec x} \vec
h\right)\left(\frac{1}{c_0}\vec k\right)-\tilde\omega\left(d_{\vec
x}\vec {\tilde u}+\left\{d_{\vec x}\vec {\tilde
u}\right\}^T\right)\cdot\left(\frac{1}{c_0}\vec
k\right)+2\tilde\omega\left(\vec k\cdot\nabla_{\vec
x}c_0\right)\left(\frac{1}{c_0}\vec
k\right)+2G\left(\frac{\tilde\omega}{c_0}\vec k\right)\\-
\frac{1}{c_0}\left\{\frac{\partial\tilde\omega}{\partial t}+\vec
h\cdot\nabla_{\vec x} \tilde\omega\right\}\vec k.
\end{multline}
However, we remind that by
\er{MaxVacFullPPNmmmffffffiuiuhjuughbghhiuijghghghhhfhhghghguygtjuuujjjkhjhjh11kklkloiouu3333}
we have
\begin{equation}\label{MaxVacFullPPNmmmffffffiuiuhjuughbghhiuijghghghhhfhhghghguygtjuuujjjkhjhjh11kklkloiouu3333jhjhhjhj}
2G=- \left(\frac{c^2_0}{\tilde\omega}\left(\frac{\partial}{\partial
t}\left\{\frac{\tilde\omega}{c^2_0}\right\}+\vec h\cdot\nabla_{\vec
x}
\left\{\frac{\tilde\omega}{c^2_0}\,\right\}\right)\right)-\left(\Div_{\vec
x}\vec h\right),
\end{equation}
Thus, inserting
\er{MaxVacFullPPNmmmffffffiuiuhjuughbghhiuijghghghhhfhhghghguygtjuuujjjkhjhjh11kklkloiouu3333jhjhhjhj}
into
\er{MaxVacFullPPNmmmffffffiuiuhjuughbghhiuijghghghhhfhhghghguygtjjjkkjjkkhhklklkjk33hhjjklikjjkjkjjjkoi;;lll1jkkjkjjljjjkjjkjkjjjkkjjjuiuuiuiuuikjkj}
gives
\begin{multline}\label{MaxVacFullPPNmmmffffffiuiuhjuughbghhiuijghghghhhfhhghghguygtjjjkkjjkkhhklklkjk33hhjjklikjjkjkjjjkoi;;lll1jkkjkjjljjjkjjkjkjjjkkjjjuiuuiuiuuikjkjjij77}
\tilde\omega\left\{\frac{\partial}{\partial
t}\left\{\frac{1}{c_0}\vec k\right\}+\Div_{\vec x}
\left\{\frac{1}{c_0}\vec k\otimes\vec {\tilde u}\right\}-d_{\vec
x}\vec {\tilde u}\cdot\left(\frac{1}{c_0}\vec
k\right)-\left(\Div_{\vec x}\left\{\frac{1}{c_0}\vec
k\right\}\right)(c_0\vec k)+2G\left(\frac{1}{c_0}\vec
k\right)\right\}=\\
\nabla_{\vec x}\tilde\omega-\left(\vec k\cdot\nabla_{\vec
x}\tilde\omega\right)\vec k\\-\tilde\omega\left(d_{\vec x}\vec
{\tilde u}+\left\{d_{\vec x}\vec {\tilde
u}\right\}^T\right)\cdot\left(\frac{1}{c_0}\vec
k\right)+2\tilde\omega\left(\vec k\cdot\nabla_{\vec
x}c_0\right)\left(\frac{1}{c_0}\vec k\right)-
\left(c_0\left(\frac{\partial}{\partial
t}\left\{\frac{\tilde\omega}{c^2_0}\right\}+\vec h\cdot\nabla_{\vec
x} \left\{\frac{\tilde\omega}{c^2_0}\,\right\}\right)\right)\vec
k\\- \frac{1}{c_0}\left\{\frac{\partial\tilde\omega}{\partial
t}+\vec h\cdot\nabla_{\vec x} \tilde\omega\right\}\vec k=
\nabla_{\vec x}\tilde\omega-\left(\vec k\cdot\nabla_{\vec
x}\tilde\omega\right)\vec k\\-\tilde\omega\left(d_{\vec x}\vec
{\tilde u}+\left\{d_{\vec x}\vec {\tilde
u}\right\}^T\right)\cdot\left(\frac{1}{c_0}\vec
k\right)+2\tilde\omega\left(\vec k\cdot\nabla_{\vec
x}c_0\right)\left(\frac{1}{c_0}\vec k\right)+
2\tilde\omega\left(\frac{1}{c_0}\left(\frac{\partial c_0}{\partial
t}+\vec h\cdot\nabla_{\vec x}
c_0\right)\right)\left(\frac{1}{c_0}\vec k\right)\\-
\frac{2}{c_0}\left\{\frac{\partial\tilde\omega}{\partial t}+\vec
h\cdot\nabla_{\vec x} \tilde\omega\right\}\vec k.
\end{multline}
So,
\begin{multline}\label{MaxVacFullPPNmmmffffffiuiuhjuughbghhiuijghghghhhfhhghghguygtjjjkkjjkkhhklklkjk33hhjjklikjjkjkjjjkoi;;lll1jkkjkjjljjjkjjkjkjjjkkjjjuiuuiuiuuikjkjjij}
\tilde\omega\left\{\frac{\partial}{\partial
t}\left\{\frac{1}{c_0}\vec k\right\}+\Div_{\vec x}
\left\{\frac{1}{c_0}\vec k\otimes\vec {\tilde u}\right\}-d_{\vec
x}\vec {\tilde u}\cdot\left(\frac{1}{c_0}\vec
k\right)-\left(\Div_{\vec x}\left\{\frac{1}{c_0}\vec
k\right\}\right)(c_0\vec k)+2G\left(\frac{1}{c_0}\vec
k\right)\right\}\\= \nabla_{\vec x}\tilde\omega-\left(\vec
k\cdot\nabla_{\vec x}\tilde\omega\right)\vec k-
\frac{2}{c_0}\left\{\frac{\partial\tilde\omega}{\partial t}+\vec
h\cdot\nabla_{\vec x} \tilde\omega\right\}\vec
k\\-\tilde\omega\left(d_{\vec x}\vec {\tilde u}+\left\{d_{\vec
x}\vec {\tilde u}\right\}^T\right)\cdot\left(\frac{1}{c_0}\vec
k\right)+ 2\tilde\omega\left(\frac{1}{c_0}\left(\frac{\partial
c_0}{\partial t}+\vec {\tilde u}\cdot\nabla_{\vec x}
c_0\right)\right)\left(\frac{1}{c_0}\vec k\right)\,.
\end{multline}
However, by
\er{MaxVacFullPPNmmmffffffiuiuhjuughbghhiuijghghghhhgjgfghjdfjgddg1133khhujuuuo33uuuiujkjkuyuyiioiiuuiiiiu35}
we have:
\begin{equation}\label{MaxVacFullPPNmmmffffffiuiuhjuughbghhiuijghghghhhgjgfghjdfjgddg1133khhujuuuo33uuuiujkjkuyuyiioiiuuiiiiu35hhjjhjhjh}
\frac{1}{\tilde\omega}\left(\frac{\partial \tilde\omega}{\partial
t}+\vec h\cdot\nabla_{\vec
x}\tilde\omega\right)=\frac{1}{c_0}\left(\frac{\partial
c_0}{\partial t}+\vec {\tilde u}\cdot\nabla_{\vec
x}c_0\right)-\frac{1}{2}\left(\vec k\cdot\left(\left\{d_{\vec x}\vec
{\tilde u}\right\}^T+d_{\vec x}\vec {\tilde
u}\right)\right)\cdot\vec k\,.
\end{equation}
Thus, inserting
\er{MaxVacFullPPNmmmffffffiuiuhjuughbghhiuijghghghhhgjgfghjdfjgddg1133khhujuuuo33uuuiujkjkuyuyiioiiuuiiiiu35hhjjhjhjh}
into
\er{MaxVacFullPPNmmmffffffiuiuhjuughbghhiuijghghghhhfhhghghguygtjjjkkjjkkhhklklkjk33hhjjklikjjkjkjjjkoi;;lll1jkkjkjjljjjkjjkjkjjjkkjjjuiuuiuiuuikjkjjij}
gives
\begin{multline}\label{MaxVacFullPPNmmmffffffiuiuhjuughbghhiuijghghghhhfhhghghguygtjjjkkjjkkhhklklkjk33hhjjklikjjkjkjjjkoi;;lll1jkkjkjjljjjkjjkjkjjjkkjjjuiuuiuiuuikjkjjijioioio}
\tilde\omega\left\{\frac{\partial}{\partial
t}\left\{\frac{1}{c_0}\vec k\right\}+\Div_{\vec x}
\left\{\frac{1}{c_0}\vec k\otimes\vec {\tilde u}\right\}-d_{\vec
x}\vec {\tilde u}\cdot\left(\frac{1}{c_0}\vec
k\right)-\left(\Div_{\vec x}\left\{\frac{1}{c_0}\vec
k\right\}\right)(c_0\vec k)+2G\left(\frac{1}{c_0}\vec
k\right)\right\}\\= \nabla_{\vec x}\tilde\omega-\left(\vec
k\cdot\nabla_{\vec x}\tilde\omega\right)\vec k
-\frac{\tilde\omega}{c_0}\left(d_{\vec x}\vec {\tilde
u}+\left\{d_{\vec x}\vec {\tilde u}\right\}^T\right)\cdot\vec
k+\frac{\tilde\omega}{c_0}\left(\left(\vec
k\cdot\left(\left\{d_{\vec x}\vec {\tilde u}\right\}^T+d_{\vec
x}\vec {\tilde u}\right)\right)\cdot\vec k\right)\vec k \,.
\end{multline}
Thus, by
\er{MaxVacFullPPNmmmffffffiuiuhjuughbghhiuijghghghhhfhhghghguygtjjjkkjjkkhhklklkjk33hhjjklikjjkjkjjjkoi;;lll1jkkjkjjljjjkjjkjkjjjkkjjjuiuuiuiuuikjkjjijioioio}
we rewrite
\er{MaxVacFullPPNmmmffffffiuiuhjuughbghhiuijghghghhhfhhghghguygtjjjkkjjkkhhklklkjk33hhjjklikjjkjkjjjkoi;;lll1jkkjkjjljjjkjjkjkjjjkk}
as:
\begin{multline}\label{MaxVacFullPPNmmmffffffiuiuhjuughbghhiuijghghghhhfhhghghguygtjjjkkjjkkhhklklkjk33hhjjklikjjkjkjjjkoi;;lll1jkkjkjjljjjkjjkjkjjjkkiii77}
\frac{c^2_0}{2\tilde\omega^2}\nabla_{\vec x}A\cdot\left(\nabla_{\vec
x}\tilde\omega-\left(\vec k\cdot\nabla_{\vec
x}\tilde\omega\right)\vec k -\frac{\tilde\omega}{c_0}\left(d_{\vec
x}\vec {\tilde u}+\left\{d_{\vec x}\vec {\tilde
u}\right\}^T\right)\cdot\vec
k+\frac{\tilde\omega}{c_0}\left(\left(\vec
k\cdot\left(\left\{d_{\vec x}\vec {\tilde u}\right\}^T+d_{\vec
x}\vec {\tilde u}\right)\right)\cdot\vec k\right)\vec k\right)\\
+\frac{c^2_0}{2\tilde\omega}A\left(\frac{\partial}{\partial
t}\left\{\frac{1}{c^2_0}G\right\}+\Div_{\vec x}
\left\{\frac{1}{c^2_0}G\vec {\tilde
u}\right\}+\frac{1}{c^2_0}(c_0\vec k)\cdot\nabla_{\vec
x}G+\frac{1}{c^2_0}G^2\right)
-\frac{c^2_0}{2\tilde\omega}\left(\Div_{\vec x}\left\{\nabla_{\vec
x}A-\left(\vec
k\cdot\nabla_{\vec x}A\right)\vec k\right\}\right)\approx\\
-i\left\{\frac{\partial A}{\partial t}
+\vec h\cdot\nabla_{\vec x}A
+\frac{1}{2}\,\left(\Div_{\vec x}\vec h \right)A\right\}-
iA\left(\frac{c_0}{\sqrt{\tilde\omega}}\left(\frac{\partial}{\partial
t}\left\{\frac{\sqrt{\tilde\omega}}{c_0}\right\}+\vec
h\cdot\nabla_{\vec x}
\left\{\frac{\sqrt{\tilde\omega}}{c_0}\,\right\}\right)\right),
\end{multline}
We rewrite it as
\begin{multline}\label{MaxVacFullPPNmmmffffffiuiuhjuughbghhiuijghghghhhfhhghghguygtjjjkkjjkkhhklklkjk33hhjjklikjjkjkjjjkoi;;lll1jkkjkjjljjjkjjkjkjjjkkiiiiouio}
-\frac{c_0}{2\tilde\omega}\left(\nabla_{\vec x}A-\left(\vec
k\cdot\nabla_{\vec x}A\right)\vec k\right)\cdot\left(\left(d_{\vec
x}\vec {\tilde u}+\left\{d_{\vec x}\vec {\tilde
u}\right\}^T\right)\cdot\vec k-\left(\left(\vec
k\cdot\left(\left\{d_{\vec x}\vec {\tilde u}\right\}^T+d_{\vec
x}\vec {\tilde u}\right)\right)\cdot\vec k\right)\vec k\right)\\
+\frac{c^2_0}{2\tilde\omega}A\left(\frac{\partial}{\partial
t}\left\{\frac{1}{c^2_0}G\right\}+\Div_{\vec x}
\left\{\frac{1}{c^2_0}G\vec {\tilde
u}\right\}+\frac{1}{c^2_0}(c_0\vec k)\cdot\nabla_{\vec
x}G+\frac{1}{c^2_0}G^2\right) -\frac{c^2_0}{2}\left(\Div_{\vec
x}\left\{\frac{1}{\tilde\omega}\left(\nabla_{\vec x}A-\left(\vec
k\cdot\nabla_{\vec x}A\right)\vec k\right)\right\}\right)\\ \approx
-i\left\{\frac{\partial A}{\partial t}
+\vec h\cdot\nabla_{\vec x}A
+\frac{1}{2}\,\left(\Div_{\vec x}\vec h \right)A\right\}-
iA\left(\frac{c_0}{\sqrt{\tilde\omega}}\left(\frac{\partial}{\partial
t}\left\{\frac{\sqrt{\tilde\omega}}{c_0}\right\}+\vec
h\cdot\nabla_{\vec x}
\left\{\frac{\sqrt{\tilde\omega}}{c_0}\,\right\}\right)\right).
\end{multline}
By
\er{MaxVacFullPPNmmmffffffiuiuhjuughbghhiuijghghghhhfhhghghguygtjjjkkjjkkhhklklkjk33hhjjklikjjkjkjjjkoi;;lll1jkkjkjjljjjkjjkjkjjjkkiiiiouio}
we have
\begin{multline}\label{MaxVacFullPPNmmmffffffiuiuhjuughbghhiuijghghghhhfhhghghguygtjjjkkjjkkhhklklkjk33hhjjklikjjkjkjjjkoi;;lll1jkkjkjjljjjkjjkjkjjjkkiii}
-\frac{c_0}{2\tilde\omega}\left(\nabla_{\vec x}A-\left(\vec
k\cdot\nabla_{\vec x}A\right)\vec k\right)\cdot\left(\left(d_{\vec
x}\vec {\tilde u}+\left\{d_{\vec x}\vec {\tilde
u}\right\}^T\right)\cdot\vec k-\left(\left(\vec
k\cdot\left(\left\{d_{\vec x}\vec {\tilde u}\right\}^T+d_{\vec
x}\vec {\tilde u}\right)\right)\cdot\vec k\right)\vec k\right)\\
+\frac{c^2_0}{2\tilde\omega}A\left(\frac{\partial}{\partial
t}\left\{\frac{1}{c^2_0}G\right\}+\Div_{\vec x}
\left\{\frac{1}{c^2_0}G\vec {\tilde
u}\right\}+\frac{1}{c^2_0}(c_0\vec k)\cdot\nabla_{\vec
x}G+\frac{1}{c^2_0}G^2\right) -\frac{c^2_0}{2}\left(\Div_{\vec
x}\left\{\frac{1}{\tilde\omega}\left(\nabla_{\vec x}A-\left(\vec
k\cdot\nabla_{\vec x}A\right)\vec k\right)\right\}\right)\\ \approx
-i\left\{\frac{\partial A}{\partial t}
+\vec h\cdot\nabla_{\vec x}A
+\frac{1}{2}\,\left(\Div_{\vec x}\vec h \right)A\right\}\\-
iA\frac{1}{2\tilde\omega}\left(\frac{\partial \tilde\omega}{\partial
t}+\vec h\cdot\nabla_{\vec
x}\tilde\omega\right)+iA\frac{1}{c_0}\left(\frac{\partial
c_0}{\partial t}+\vec h\cdot\nabla_{\vec x}c_0\right).
\end{multline}
On the other hand, by
\er{MaxVacFullPPNmmmffffffiuiuhjuughbghhiuijghghghhhgjgfghjdfjgddg1133khhujuuuo33uuuiujkjkuyuyiioiiuuiiiiu35}
we have
\begin{equation}\label{MaxVacFullPPNmmmffffffiuiuhjuughbghhiuijghghghhhgjgfghjdfjgddg1133khhujuuuo33uuuiujkjkuyuyiioiiuuiiiiu35jhjgh}
\frac{1}{2\tilde\omega}\left(\frac{\partial \tilde\omega}{\partial
t}+\vec h\cdot\nabla_{\vec
x}\tilde\omega\right)=\frac{1}{2c_0}\left(\frac{\partial
c_0}{\partial t}+\vec {\tilde u}\cdot\nabla_{\vec
x}c_0\right)-\frac{1}{4}\left(\vec k\cdot\left(\left\{d_{\vec x}\vec
{\tilde u}\right\}^T+d_{\vec x}\vec {\tilde
u}\right)\right)\cdot\vec k\,.
\end{equation}
Therefore, inserting
\er{MaxVacFullPPNmmmffffffiuiuhjuughbghhiuijghghghhhgjgfghjdfjgddg1133khhujuuuo33uuuiujkjkuyuyiioiiuuiiiiu35jhjgh}
into
\er{MaxVacFullPPNmmmffffffiuiuhjuughbghhiuijghghghhhfhhghghguygtjjjkkjjkkhhklklkjk33hhjjklikjjkjkjjjkoi;;lll1jkkjkjjljjjkjjkjkjjjkkiii}
gives
\begin{multline}\label{MaxVacFullPPNmmmffffffiuiuhjuughbghhiuijghghghhhfhhghghguygtjjjkkjjkkhhklklkjk33hhjjklikjjkjkjjjkoi;;lll1jkkjkjjljjjkjjkjkjjjkkiiihh}
-\frac{c_0}{2\tilde\omega}\left(\nabla_{\vec x}A-\left(\vec
k\cdot\nabla_{\vec x}A\right)\vec k\right)\cdot\left(\left(d_{\vec
x}\vec {\tilde u}+\left\{d_{\vec x}\vec {\tilde
u}\right\}^T\right)\cdot\vec k-\left(\left(\vec
k\cdot\left(\left\{d_{\vec x}\vec {\tilde u}\right\}^T+d_{\vec
x}\vec {\tilde u}\right)\right)\cdot\vec k\right)\vec k\right)\\
+\frac{c^2_0}{2\tilde\omega}A\left(\frac{\partial}{\partial
t}\left\{\frac{1}{c^2_0}G\right\}+\Div_{\vec x}
\left\{\frac{1}{c^2_0}G\vec {\tilde
u}\right\}+\frac{1}{c^2_0}(c_0\vec k)\cdot\nabla_{\vec
x}G+\frac{1}{c^2_0}G^2\right) \\-\frac{c^2_0}{2}\left(\Div_{\vec
x}\left\{\frac{1}{\tilde\omega}\left(\nabla_{\vec x}A-\left(\vec
k\cdot\nabla_{\vec x}A\right)\vec k\right)\right\}\right) \approx
-i\left\{\frac{\partial A}{\partial t}
+\vec h\cdot\nabla_{\vec x}A
+\frac{1}{2}\,\left(\Div_{\vec x}\vec h \right)A\right\}\\-
iA\frac{1}{2c_0}\left(c_0\vec k\cdot\nabla_{\vec
x}c_0\right)+iA\frac{1}{2c_0}\left(\frac{\partial c_0}{\partial
t}+\vec h\cdot\nabla_{\vec x}c_0\right)+ \frac{iA}{4}\left(\vec
k\cdot\left(\left\{d_{\vec x}\vec {\tilde u}\right\}^T+d_{\vec
x}\vec {\tilde u}\right)\right)\cdot\vec k.
\end{multline}
Moreover, as before, we can easily deduce that
\er{MaxVacFullPPNmmmffffffiuiuhjuughbghhiuijghghghhhfhhghghguygtjjjkkjjkkhhklklkjk33hhjjklikjjkjkjjjkoi;;lll1jkkjkjjljjjkjjkjkjjjkkiiihh}
invariant under the change of non-inertial cartesian coordinate
system, in the case that under such change we have $A'=A$ and
$S'=S$.

In particular, if in some Cartesian coordinate system we assume time
independent settings of the problem so that $\frac{\partial
S}{\partial t}=c$, $\frac{\partial c_0}{\partial t}=0$ and
$\frac{\partial \vec {\tilde u}}{\partial t}=0$, then we also have
$\frac{\partial {\vec k}}{\partial t}=\frac{\partial  {\vec
h}}{\partial t}=0$ and $\frac{\partial \tilde\omega}{\partial
t}=\frac{\partial G}{\partial t}=0$, and therefore, time-independent
solutions $A:=A(\vec x)$, of
\er{MaxVacFullPPNmmmffffffiuiuhjuughbghhiuijghghghhhfhhghghguygtjjjkkjjkkhhklklkjk33hhjjklikjjkjkjjjkoi;;lll1jkkjkjjljjjkjjkjkjjjkkiiihh}
are admitted and they satisfy
\begin{multline}\label{MaxVacFullPPNmmmffffffiuiuhjuughbghhiuijghghghhhfhhghghguygtjjjkkjjkkhhklklkjk33hhjjklikjjkjkjjjkoi;;lll1jkkjkjjljjjkjjkjkjjjkkiiihhukuku}
-\frac{c_0}{2\tilde\omega}\left(\nabla_{\vec x}A-\left(\vec
k\cdot\nabla_{\vec x}A\right)\vec k\right)\cdot\left(\left(d_{\vec
x}\vec {\tilde u}+\left\{d_{\vec x}\vec {\tilde
u}\right\}^T\right)\cdot\vec k-\left(\left(\vec
k\cdot\left(\left\{d_{\vec x}\vec {\tilde u}\right\}^T+d_{\vec
x}\vec {\tilde u}\right)\right)\cdot\vec k\right)\vec k\right)\\
+\frac{c^2_0}{2\tilde\omega}A\left(\Div_{\vec x}
\left\{\frac{1}{c^2_0}G\vec {\tilde
u}\right\}+\frac{1}{c^2_0}(c_0\vec k)\cdot\nabla_{\vec
x}G+\frac{1}{c^2_0}G^2\right) -\frac{c^2_0}{2}\left(\Div_{\vec
x}\left\{\frac{1}{\tilde\omega}\left(\nabla_{\vec x}A-\left(\vec
k\cdot\nabla_{\vec x}A\right)\vec k\right)\right\}\right)\\ \approx
-i\left\{\vec h\cdot\nabla_{\vec x}A
+\frac{1}{2}\,\left(\Div_{\vec x}\vec h \right)A\right\}\\-
iA\frac{1}{2c_0}\left(c_0\vec k\cdot\nabla_{\vec
x}c_0\right)+iA\frac{1}{2c_0}\left(\vec h\cdot\nabla_{\vec
x}c_0\right)+ \frac{iA}{4}\left(\vec k\cdot\left(\left\{d_{\vec
x}\vec {\tilde u}\right\}^T+d_{\vec x}\vec {\tilde
u}\right)\right)\cdot\vec k.
\end{multline}

Finally, from now we assume that $\vec {\tilde u}$ and $c_0$ vary
sufficiently slowly in space and time variables, so that the
following approximations are valid:
\begin{equation}\label{MaxVacFullPPNmmmffffffhhtygghGGGGyuhggghghghffgiooiiojhjjhhjuiuiugughhjggjjhk}
\frac{|\nabla_{\vec x}c_0|+\frac{1}{c_0}\left|\frac{\partial
c_0}{\partial t}+\vec {\tilde u}\cdot\nabla_{\vec x}c_0\right|}{
c_0}+\frac{\left| d_{\vec x}\vec {\tilde u}+\left\{d_{\vec x}\vec
{\tilde u}\right\}^T\right|}{
c_0}\ll\frac{\left(| \nabla_{\vec
x}A|+\frac{1}{c_0}\left|\frac{\partial A}{\partial t}+\vec {\tilde
u}\cdot\nabla_{\vec x}A\right|\right)}{|A|},
\end{equation}
In particular, if in some Cartesian coordinate system we have
\begin{equation}\label{MaxVacFullPPNmmmffffffiuiuhjuughbghhiuijghghghhjhjhhghyuyhhjhjihikjjjjkjljklghhgggh}
\begin{cases}
\frac{\left| d_{\vec x}\vec {\tilde u}+\left\{d_{\vec x}\vec {\tilde
u}\right\}^T\right|^2}{
|\vec {\tilde u}|^2}\ll\frac{\left(| \nabla_{\vec
x}A|+\frac{1}{c_0}\left|\frac{\partial A}{\partial t}+\vec {\tilde
u}\cdot\nabla_{\vec x}A\right|\right)^2}{|A|^2}\\
\frac{\left|\vec {\tilde u}\right|^2}{
c^2_0}\ll 1\\
\frac{|\nabla_{\vec x}c_0|+\frac{1}{c_0}\left|\frac{\partial
c_0}{\partial t}+\vec {\tilde u}\cdot\nabla_{\vec x}c_0\right|}{
c_0}\ll\frac{\left(| \nabla_{\vec
x}A|+\frac{1}{c_0}\left|\frac{\partial A}{\partial t}+\vec {\tilde
u}\cdot\nabla_{\vec x}A\right|\right)}{|A|},
\end{cases}
\end{equation}
then
\er{MaxVacFullPPNmmmffffffhhtygghGGGGyuhggghghghffgiooiiojhjjhhjuiuiugughhjggjjhk}
indeed holds! Then, by
\er{MaxVacFullPPNmmmffffffhhtygghGGGGyuhggghghghffgiooiiojhjjhhjuiuiugughhjggjjhk}
we rewrite and simplify
\er{MaxVacFullPPNmmmffffffiuiuhjuughbghhiuijghghghhhfhhghghguygtjjjkkjjkkhhklklkjk33hhjjklikjjkjkjjjkoi;;lll1jkkjkjjljjjkjjkjkjjjkkiiihh}
as
\begin{multline}\label{MaxVacFullPPNmmmffffffiuiuhjuughbghhiuijghghghhhfhhghghguygtjjjkkjjkkhhklklkjk33hhjjklikjjkjkjjjkoi;;lll1jkkjkjjljjjkjjkjkjjjkkiiihhihhjh}
\frac{1}{2\tilde\omega}A\left(\frac{\partial G}{\partial
t}+\Div_{\vec x} \left\{G\vec {\tilde u}\right\}+(c_0\vec
k)\cdot\nabla_{\vec x}G+G^2\right)
\\-\frac{1}{2}\left(\Div_{\vec
x}\left\{\frac{c^2_0}{\tilde\omega}\left(\nabla_{\vec x}A-\left(\vec
k\cdot\nabla_{\vec x}A\right)\vec k\right)\right\}\right) \approx
-i\left\{\frac{\partial A}{\partial t}
+\vec h\cdot\nabla_{\vec x}A
+\frac{1}{2}\,\left(\Div_{\vec x}\vec h \right)A\right\}.
\end{multline}
Moreover, as before, we can easily deduce that
\er{MaxVacFullPPNmmmffffffiuiuhjuughbghhiuijghghghhhfhhghghguygtjjjkkjjkkhhklklkjk33hhjjklikjjkjkjjjkoi;;lll1jkkjkjjljjjkjjkjkjjjkkiiihhihhjh}
invariant under the change of non-inertial cartesian coordinate
system, in the case that under such change we have $A'=A$ and
$S'=S$.

Furthermore, in time independent settings we rewrite
\er{MaxVacFullPPNmmmffffffiuiuhjuughbghhiuijghghghhhfhhghghguygtjjjkkjjkkhhklklkjk33hhjjklikjjkjkjjjkoi;;lll1jkkjkjjljjjkjjkjkjjjkkiiihhukuku}
as
\begin{multline}\label{MaxVacFullPPNmmmffffffiuiuhjuughbghhiuijghghghhhfhhghghguygtjjjkkjjkkhhklklkjk33hhjjklikjjkjkjjjkoi;;lll1jkkjkjjljjjkjjkjkjjjkkiiihhihhjhijjk}
\frac{1}{2\tilde\omega}A\left(\Div_{\vec x} \left\{G\vec {\tilde
u}\right\}+(c_0\vec k)\cdot\nabla_{\vec x}G+G^2\right)
\\-\frac{1}{2}\left(\Div_{\vec
x}\left\{\frac{c^2_0}{\tilde\omega}\left(\nabla_{\vec x}A-\left(\vec
k\cdot\nabla_{\vec x}A\right)\vec k\right)\right\}\right) \approx
-i\left\{\vec h\cdot\nabla_{\vec x}A
+\frac{1}{2}\,\left(\Div_{\vec x}\vec h \right)A\right\}.
\end{multline}
In particular, in the case $\vec {\tilde u}=0$ and constant $\vec
k$, $c_0$ and $\tilde\omega$,
\er{MaxVacFullPPNmmmffffffiuiuhjuughbghhiuijghghghhhfhhghghguygtjjjkkjjkkhhklklkjk33hhjjklikjjkjkjjjkoi;;lll1jkkjkjjljjjkjjkjkjjjkkiiihhihhjhijjk}
simplifies as the paraxial approximation of the Helmholtz equation,
with respect to direction $(-\vec k)$, which is common in Optics:
\begin{equation}\label{MaxVacFullPPNmmmffffffiuiuhjuughbghhiuijghghghhhfhhghghguygtjjjkkjjkkhhklklkjk33hhjjklikjjkjkjjjkoi;;lll1jkkjkjjljjjkjjkjkjjjkkiiihhihhjhijjkjkjkjkklk}
\frac{c_0}{2\tilde\omega}\left(\Div_{\vec
x}\left\{\left(\nabla_{\vec x}A-\left(\vec h\cdot\nabla_{\vec
x}A\right)\vec h\right)\right\}\right) \approx
-i\vec k\cdot\nabla_{\vec x}A .
\end{equation}

Next, note that, the time dependent equation
\er{MaxVacFullPPNmmmffffffiuiuhjuughbghhiuijghghghhhfhhghghguygtjjjkkjjkkhhklklkjk33hhjjklikjjkjkjjjkoi;;lll1jkkjkjjljjjkjjkjkjjjkkiiihhihhjh}
is very similar to the Schr\"{o}dinger equation and similarly can be
written in the form
\begin{equation}\label{MaxVacFullPPNmmmffffffiuiuhjuughbghhiuijghghghhhfhhghghguygtjjjkkjjkkhhklklkjk33hhjjklikjjkjkjjjkoi;;lll1jkkjkjjljjjkjjkjkjjjkkiiihhihhjhuouo}
 i\hbar\frac{\partial A}{\partial t} =\hat H\cdot A.
\end{equation}
where
\begin{multline}\label{MaxVacFullPPNmmmffffffiuiuhjuughbghhiuijghghghhhfhjkkhkhghghguygtjjjkkjjkkhhklklkjk33hhjjklikjjkjkjjjkoi;;lll1jkkjkjjljjjkjjkjkjjjkkiiihhihhjhkklhkhk}
\hat H\cdot A:= \frac{\hbar}{2}\left(\Div_{\vec
x}\left\{\frac{c^2_0}{\tilde\omega}\left(\nabla_{\vec x}A-\left(\vec
k\cdot\nabla_{\vec x}A\right)\vec k\right)\right\}\right)
-i\hbar\left\{
\vec h\cdot\nabla_{\vec x}A
+\frac{1}{2}\,\left(\Div_{\vec x}\vec h \right)A\right\}\\
-\frac{\hbar}{2\tilde\omega}\left(\frac{\partial G}{\partial
t}+\Div_{\vec x} \left\{G\vec {\tilde u}\right\}+(c_0\vec
k)\cdot\nabla_{\vec x}G+G^2\right)A
\end{multline}
is a self-adjoint linear operator on the complex Hilbert space.

Finally, solving either
\er{MaxVacFullPPNmmmffffffiuiuhjuughbghhiuijghghghhhfhhghghguygtjjjkkjjkkhhklklkjk33hhjjklikjjkjkjjjkoi;;lll1jkkjkjjljjjkjjkjkjjjkkiiihh}
or
\er{MaxVacFullPPNmmmffffffiuiuhjuughbghhiuijghghghhhfhhghghguygtjjjkkjjkkhhklklkjk33hhjjklikjjkjkjjjkoi;;lll1jkkjkjjljjjkjjkjkjjjkkiiihhihhjh}
and inserting the solution into
\er{MaxVacFullPPNmmmffffffiuiuhjuughbghhuiiujjjjjjjjgjjgiuiuui} we
established $U(\vec x,t)$ in
\er{MaxVacFullPPNmmmffffffiuiuhjuughbghhjjjoojouiuiu} in the
\underline{delicate} Geometric Optics approximation
\er{slochangGGaaffgfgjhjggh}.

\subsubsection{Wave optics inside a moving medium and/or in the
presence of gravitational field}\label{jghggghghgh} Assume that in
some inertial or non-inertial cartesian coordinate system a
\underline{complex valued} scalar field $U:=U(\vec x,t)$,
characterizing some wave, satisfies the wave equation
\er{MaxVacFullPPNmmmffffffiuiuhjuughbghhjj}:
\begin{equation}\label{MaxVacFullPPNmmmffffffiuiuhjuughbghhjjhh11}
\frac{\partial}{\partial
t}\left\{\frac{1}{c^2_0}\left(\frac{\partial U}{\partial t}+\vec
{\tilde u}\cdot\nabla_{\vec x}U\right)\right\}+\Div_{\vec x}
\left\{\frac{1}{c^2_0}\left(\frac{\partial U}{\partial t}+\vec
{\tilde u}\cdot\nabla_{\vec x} U\right)\vec {\tilde
u}\right\}-\Delta_{\vec x}U=0,
\end{equation}
where, as before, $\vec {\tilde u}:=\vec {\tilde u}(\vec x,t)$ is
some moderately changing (in space and in time) speed-like vector
field and $c_0:=c_0(\vec x,t)>0$ is a moderately changing scalar
quantity, that we call wave propagation speed. In particular, if we
assume that the fields $\vec {\tilde u}$ and $c_0$ are independent
on the time variable, and we consider the case of a monochromatic
wave of the constant frequency $\nu=\frac{\omega}{2\pi}$:
\begin{equation}\label{MaxVacFullPPNmmmffffffiuiuhjuughbghhiuijghghghhjhjhhghyuyhhjhjihikjjjjkjljklhh11}
U(\vec x,t):=\mathcal{U}(\vec x)e^{i\omega t}\quad\quad\forall(\vec
x,t),
\end{equation}
where $\mathcal{U}:=\mathcal{U}(\vec x)$ is a complex-valued scalar
field, independent of time, then inserting
\er{MaxVacFullPPNmmmffffffiuiuhjuughbghhiuijghghghhjhjhhghyuyhhjhjihikjjjjkjljklhh11}
into \er{MaxVacFullPPNmmmffffffiuiuhjuughbghhjjhh11} gives the
following time independent equation:
\begin{equation}\label{MaxVacFullPPNmmmffffffiuiuhjuughbghhjjuggk;klklkljkl}
\frac{i\omega}{c^2_0(\vec x)}\left(i\omega \mathcal{U}(\vec x)+\vec
{\tilde u}(\vec x)\cdot\nabla_{\vec x}\mathcal{U}(\vec
x)\right)+\Div_{\vec x} \left\{\frac{1}{c^2_0(\vec x)}\left(i\omega
\mathcal{U}(\vec x)+\vec {\tilde u}(\vec x)\cdot\nabla_{\vec x}
\mathcal{U}(\vec x)\right)\vec {\tilde u}(\vec
x)\right\}-\Delta_{\vec x}\mathcal{U}(\vec x)=0.
\end{equation}

Next, suppose that, although we cannot consider the Geometric Optics
approximation \er{slochangGGaaffgfgjhjggh}, we can consider,
however, the following assumption:
\begin{equation}\label{iggguggkhhhh}
\vec {\tilde u}(\vec x,t)=\nabla_{\vec x}\tilde Z(\vec x,t)
\quad\quad\forall(\vec x,t),
\end{equation}
where $\tilde Z:=\tilde Z(\vec x,t)$ is some real-valued scalar
function, which satisfies
\begin{equation}\label{iggguggkhhhh22}
\frac{1}{c^2_0}\left|\frac{\partial \tilde Z}{\partial
t}\right|+\frac{1}{c^2_0}\left|\nabla_{\vec x}\tilde Z\right|^2\ll
1.
\end{equation}
Note that \er{iggguggkhhhh} is valid in the particular case of the
electromagnetic waves, propagating in vacuum, in the
\underline{inertial} or-more generally non-rotating cartesian
coordinate system (where we indeed have $\vec {\tilde u}=\vec v$ and
$curl_{\vec x}\vec v=0$). On the other hand, \er{iggguggkhhhh22}
means that $\frac{\left|\vec {\tilde u}\right|^2}{c^2_0}\ll 1$ and
$\vec {\tilde u}$ is either independent of time or depends on time
slowly. Moreover, assume that $c_0$ is a constant in the given
region.

Next we do the change of variables $(\vec x,t)\to(\vec z,\tau)$ as
\begin{equation}\label{hvkgkjgkjkhhhhjjkhkhkkjjjjhgj}
\begin{cases}
\vec z:=\vec x,\\
\tau:=t+\frac{1}{c^2_0
}\tilde Z(\vec x,t),
\end{cases}
\end{equation}
so that, we can find a complex-valued function $V:=V(\vec z,\tau)$
depending on  $\vec z\in\mathbb{R}^3$ and a real $\tau$, which
satisfies:
\begin{equation}\label{hvkgkjgkjkhhhhjjkhkhk}
U(\vec x,t):=V\left(\vec x,t+\frac{1}{c^2_0
}\tilde Z(\vec
x,t)\right).
\end{equation}
Next we differentiate \er{hvkgkjgkjkhhhhjjkhkhk}, assuming that
$c_0$ vary very slowly and we can consider it to be a constant in
the differentiation. Thus, by the Chain Rule we deduce:
\begin{equation}\label{hvkgkjgkjkhhhhjjkhkhkhjgh}
\begin{cases}
\nabla_{\vec x}U(\vec x,t)=\nabla_{\vec z}V\left(\vec
x,t+\frac{1}{c^2_0}\tilde Z(\vec
x,t)\right)+\frac{1}{c^2_0}\frac{\partial V}{\partial\tau}\left(\vec
x,t+\frac{1}{c^2_0}\tilde Z(\vec
x,t)\right)\,\nabla_{\vec x}\tilde Z(\vec x,t)\\
\frac{\partial U}{\partial t}(\vec x,t)=\frac{\partial
V}{\partial\tau}\left(\vec x,t+\tilde Z(\vec
x,t)\right)\left(1+\frac{1}{c^2_0}\frac{\partial \tilde Z}{\partial
t}(\vec x,t)\right)\,.
\end{cases}
\end{equation}
Therefore, again  by the Chain Rule we infer:
\begin{multline}\label{hvkgkjgkjkhhhhjjkhkhkjkhkhkgjggjjg}
\nabla^2_{\vec x}U(\vec x,t)=\nabla^2_{\vec z}V\left(\vec
x,t+\frac{1}{c^2_0}\tilde Z(\vec
x,t)\right)+\frac{1}{c^2_0}\nabla_{\vec z}\left(\frac{\partial
V}{\partial\tau}\right)\left(\vec x,t+\frac{1}{c^2_0}\tilde Z(\vec
x,t)\right)\otimes\nabla_{\vec x}\tilde Z(\vec
x,t)\\+\frac{1}{c^2_0}\nabla_{\vec x}\tilde Z(\vec
x,t)\otimes\nabla_{\vec z}\left(\frac{\partial
V}{\partial\tau}\right)\left(\vec x,t+\frac{1}{c^2_0}\tilde Z(\vec
x,t)\right)+\frac{1}{c^4_0}\frac{\partial^2
V}{\partial\tau^2}\left(\vec x,t+\frac{1}{c^2_0}\tilde Z(\vec
x,t)\right)\,\nabla_{\vec x}\tilde Z(\vec x,t)\otimes\nabla_{\vec
x}\tilde Z(\vec x,t)\\+\frac{1}{c^2_0}\frac{\partial
V}{\partial\tau}\left(\vec x,t+\frac{1}{c^2_0}\tilde Z(\vec
x,t)\right)\,\nabla^2_{\vec x}\tilde Z(\vec x,t).
\end{multline}
In particular,
\begin{multline}\label{hvkgkjgkjkhhhhjjkhkhkjkhkhkgjggjjghggjk}
\Delta_{\vec x}U(\vec x,t)=\Delta_{\vec z}V\left(\vec
x,t+\frac{1}{c^2_0}\tilde Z(\vec
x,t)\right)+\frac{2}{c^2_0}\nabla_{\vec z}\left(\frac{\partial
V}{\partial\tau}\right)\left(\vec x,t+\frac{1}{c^2_0}\tilde Z(\vec
x,t)\right)\cdot\nabla_{\vec x}\tilde Z(\vec
x,t)\\+\frac{1}{c^4_0}\frac{\partial^2 V}{\partial\tau^2}\left(\vec
x,t+\frac{1}{c^2_0}\tilde Z(\vec x,t)\right)\left|\nabla_{\vec
x}\tilde Z(\vec x,t)\right|^2+\frac{1}{c^2_0}\frac{\partial
V}{\partial\tau}\left(\vec x,t+\frac{1}{c^2_0}\tilde Z(\vec
x,t)\right)\Delta_{\vec x}\tilde Z(\vec x,t).
\end{multline}
Moreover, again by the Chain Rule we have
\begin{multline}\label{hvkgkjgkjkhhhhjjkhkhkjkhkhkgjggjjgjujkkhhk}
\nabla_{\vec x}\left(\frac{\partial U}{\partial t}\right)(\vec
x,t)=\left(1+\frac{1}{c^2_0}\frac{\partial \tilde Z}{\partial
t}(\vec x,t)\right)\nabla_{\vec z}\left(\frac{\partial
V}{\partial\tau}\right)\left(\vec x,t+\frac{1}{c^2_0}\tilde Z(\vec
x,t)\right)\\+\frac{1}{c^2_0}\left(1+\frac{1}{c^2_0}\frac{\partial
\tilde Z}{\partial t}(\vec x,t)\right)\frac{\partial^2
V}{\partial\tau^2}\left(\vec x,t+\frac{1}{c^2_0}\tilde Z(\vec
x,t)\right)\nabla_{\vec x}\tilde Z(\vec
x,t)\\+\frac{1}{c^2_0}\frac{\partial V}{\partial\tau}\left(\vec
x,t+\frac{1}{c^2_0}\tilde Z(\vec x,t)\right)\nabla_{\vec
x}\left(\frac{\partial \tilde Z}{\partial t}(\vec x,t)\right),
\end{multline}
and
\begin{equation}\label{hvkgkjgkjkhhhhjjkhkhkjkhkhkgjggjjgjujkkhhkhkhkhhhjg}
\frac{\partial^2 U}{\partial t^2}(\vec x,t)=\\ \frac{\partial^2
V}{\partial\tau^2}\left(\vec x,t+\frac{1}{c^2_0}\tilde Z(\vec
x,t)\right)\left(1+\frac{1}{c^2_0}\frac{\partial \tilde Z}{\partial
t}(\vec x,t)\right)^2+\frac{1}{c^2_0}\frac{\partial
V}{\partial\tau}\left(\vec x,t+\frac{1}{c^2_0}\tilde Z(\vec
x,t)\right)\,\frac{\partial^2 \tilde Z}{\partial t^2}(\vec x,t).
\end{equation}
Then, by \er{hvkgkjgkjkhhhhjjkhkhkjkhkhkgjggjjghggjk},
\er{hvkgkjgkjkhhhhjjkhkhkjkhkhkgjggjjgjujkkhhk} and
\er{hvkgkjgkjkhhhhjjkhkhkjkhkhkgjggjjgjujkkhhkhkhkhhhjg} together
with \er{iggguggkhhhh22} we have
\begin{multline}\label{hvkgkjgkjkhhhhjjkhkhkjkhkhkgjggjjghggjkkkhhk}
\Delta_{\vec x}U(\vec x,t)=\Delta_{\vec z}V\left(\vec
x,t+\frac{1}{c^2_0}\tilde Z(\vec
x,t)\right)+\frac{2}{c^2_0}\nabla_{\vec z}\left(\frac{\partial
V}{\partial\tau}\right)\left(\vec x,t +\frac{1}{c^2_0}\tilde Z(\vec
x,t)\right)\cdot\nabla_{\vec x}\tilde Z(\vec x,t)
\\
+\frac{1}{c^2_0}\frac{\partial V}{\partial\tau}\left(\vec
x,t+\frac{1}{c^2_0}\tilde Z(\vec x,t)\right)\Delta_{\vec x}\tilde
Z(\vec x,t)+\frac{1}{c^4_0}\frac{\partial^2
V}{\partial\tau^2}\left(\vec x,t+\frac{1}{c^2_0}\tilde Z(\vec
x,t)\right)\left|\nabla_{\vec x}\tilde Z(\vec x,t)\right|^2,
\end{multline}
\begin{multline}\label{hvkgkjgkjkhhhhjjkhkhkjkhkhkgjggjjgjujkkhhkhhjh}
\nabla_{\vec x}\left(\frac{\partial U}{\partial t}\right)(\vec
x,t)\approx
\nabla_{\vec z}\left(\frac{\partial
V}{\partial\tau}\right)\left(\vec x,t+\frac{1}{c^2_0}\tilde Z(\vec
x,t)\right)+\frac{1}{c^2_0}
\frac{\partial^2 V}{\partial\tau^2}\left(\vec
x,t+\frac{1}{c^2_0}\tilde Z(\vec x,t)\right)\nabla_{\vec x}\tilde
Z(\vec x,t)\\+\frac{1}{c^2_0}\frac{\partial
V}{\partial\tau}\left(\vec x,t+\frac{1}{c^2_0}\tilde Z(\vec
x,t)\right)\nabla_{\vec x}\left(\frac{\partial \tilde Z}{\partial
t}(\vec x,t)\right),
\end{multline}
and
\begin{equation}\label{hvkgkjgkjkhhhhjjkhkhkjkhkhkgjggjjgjujkkhhkhkhkhhhjgjgj}
\frac{\partial^2 U}{\partial t^2}(\vec x,t)\approx\\
\frac{\partial^2 V}{\partial\tau^2}\left(\vec
x,t+\frac{1}{c^2_0}\tilde Z(\vec
x,t)\right)+\frac{1}{c^2_0}\frac{\partial V}{\partial\tau}\left(\vec
x,t+\frac{1}{c^2_0}\tilde Z(\vec x,t)\right)\,\frac{\partial^2
\tilde Z}{\partial t^2}(\vec x,t).
\end{equation}
Moreover, by \er{hvkgkjgkjkhhhhjjkhkhkhjgh} together with
\er{iggguggkhhhh22} we have
\begin{equation}\label{hvkgkjgkjkhhhhjjkhkhkhjghhkhkkkhkh}
\begin{cases}
\nabla_{\vec x}U(\vec x,t)=\nabla_{\vec z}V\left(\vec
x,t+\frac{1}{c^2_0}\tilde Z(\vec
x,t)\right)+\frac{1}{c^2_0}\frac{\partial V}{\partial\tau}\left(\vec
x,t+\frac{1}{c^2_0}\tilde Z(\vec
x,t)\right)\,\nabla_{\vec x}\tilde Z(\vec x,t)\\
\frac{\partial U}{\partial t}(\vec x,t)\approx\frac{\partial
V}{\partial\tau}\left(\vec x,t+\frac{1}{c^2_0}\tilde Z(\vec
x,t)\right)\,.
\end{cases}
\end{equation}
Thus, by \er{iggguggkhhhh} and \er{iggguggkhhhh22} using
\er{hvkgkjgkjkhhhhjjkhkhkjkhkhkgjggjjghggjkkkhhk},
\er{hvkgkjgkjkhhhhjjkhkhkjkhkhkgjggjjgjujkkhhkhhjh},
\er{hvkgkjgkjkhhhhjjkhkhkjkhkhkgjggjjgjujkkhhkhkhkhhhjgjgj} and
\er{hvkgkjgkjkhhhhjjkhkhkhjghhkhkkkhkh} we deduce
\begin{multline}\label{MaxVacFullPPNmmmffffffiuiuhjuughbghhjjkhhkhh}
\frac{1}{c^2_0}\left(\frac{\partial}{\partial t}\left(\frac{\partial
U}{\partial t}+\vec {\tilde u}\cdot\nabla_{\vec x}U\right)+div_{\vec
x} \left\{\left(\frac{\partial U}{\partial t}+\vec {\tilde
u}\cdot\nabla_{\vec x} U\right)\vec {\tilde
u}\right\}\right)-\Delta_{\vec x}U\approx\\
\frac{1}{c^2_0}\left(\frac{\partial^2 U}{\partial t^2}+2\vec {\tilde
u}\cdot\nabla_{\vec x}\left(\frac{\partial U}{\partial
t}\right)+\left(div_{\vec x}\vec {\tilde u}\right)\frac{\partial
U}{\partial t}\right)-\Delta_{\vec x}U\\ \approx
\frac{1}{c^2_0}\frac{\partial^2 V}{\partial\tau^2}\left(\vec
x,t+\frac{1}{c^2_0}\tilde Z(\vec x,t)\right)-\Delta_{\vec
z}V\left(\vec
x,t+\frac{1}{c^2_0}\tilde Z(\vec x,t)\right)\\
+\frac{1}{c^4_0}\frac{\partial V}{\partial\tau}\left(\vec
x,t+\frac{1}{c^2_0}\tilde Z(\vec x,t)\right)\,\frac{\partial^2
\tilde Z}{\partial t^2}(\vec x,t) +\frac{2}{c^4_0}\frac{\partial
V}{\partial\tau}\left(\vec x,t+\frac{1}{c^2_0}\tilde Z(\vec
x,t)\right)\nabla_{\vec x}\left(\frac{\partial \tilde Z}{\partial
t}(\vec x,t)\right)\cdot\nabla_{\vec x}\tilde Z(\vec
x,t)\\+\frac{1}{c^4_0}\frac{\partial^2 V}{\partial\tau^2}\left(\vec
x,t+\frac{1}{c^2_0}\tilde Z(\vec x,t)\right)\left|\nabla_{\vec
x}\tilde Z(\vec x,t)\right|^2
\approx\\
\frac{1}{c^2_0}\frac{\partial^2 V}{\partial\tau^2}\left(\vec
x,t+\frac{1}{c^2_0}\tilde Z(\vec x,t)\right)-\Delta_{\vec
z}V\left(\vec x,t+\frac{1}{c^2_0}\tilde Z(\vec x,t)\right) =
\frac{1}{c^2_0}\frac{\partial^2 V}{\partial\tau^2}\left(\vec
z,\tau\right)-\Delta_{\vec z}V\left(\vec z,\tau\right).
\end{multline}
Therefore, inserting
\er{MaxVacFullPPNmmmffffffiuiuhjuughbghhjjkhhkhh} into
\er{MaxVacFullPPNmmmffffffiuiuhjuughbghhjjhh11} we deduce that the
function $V(\vec z,\tau)$ solves the standard type of the wave
equation of the form:
\begin{equation}\label{MaxVacFullPPNmmmffffffiuiuhjuughbghhjjkhhkhhhhhj}
\frac{1}{c^2_0}\frac{\partial^2 V}{\partial\tau^2}\left(\vec
z,\tau\right)-\Delta_{\vec z}V\left(\vec z,\tau\right)\approx 0.
\end{equation}
Fortunately, the plenty of tools for analytical resolution of the
\underline{standard} wave equation
\er{MaxVacFullPPNmmmffffffiuiuhjuughbghhjjkhhkhhhhhj} is well known.
In particular, if we assume that the fields $\vec {\tilde u}$ and
$c_0$ are independent of the time variable, and we consider the case
of a monochromatic wave of the constant frequency
$\nu=\frac{\omega}{2\pi}$, as in
\er{MaxVacFullPPNmmmffffffiuiuhjuughbghhiuijghghghhjhjhhghyuyhhjhjihikjjjjkjljklhh11},
then, we can choose $\tilde Z$ to be independent of time and
moreover, inserting
\er{MaxVacFullPPNmmmffffffiuiuhjuughbghhiuijghghghhjhjhhghyuyhhjhjihikjjjjkjljklhh11}
into \er{hvkgkjgkjkhhhhjjkhkhk} gives:
\begin{equation}\label{hvkgkjgkjkhhhhjjkhkhkhkhhjhgjojjo}
\mathcal{U}(\vec x)e^{i\omega t}:=V\left(\vec
x,t+\frac{1}{c^2_0}\tilde Z(\vec x)\right),
\end{equation}
and so
\begin{equation}\label{hvkgkjgkjkhhhhjjkhkhkhkhhjhgjkjiuiiyjljl}
V\left(\vec z,\tau\right)=\mathcal{V}\left(\vec z\right)e^{i\omega
\tau}\quad\quad\text{and}\quad\quad \mathcal{U}(\vec
x):=\mathcal{V}\left(\vec x\right)e^{\frac{i\omega }{c^2_0}\tilde
Z(\vec x)},
\end{equation}
where $\mathcal{V}:=\mathcal{V}(\vec z)$ is a complex-valued scalar
field, independent of time. Moreover, inserting
\er{hvkgkjgkjkhhhhjjkhkhkhkhhjhgjkjiuiiyjljl} into
\er{MaxVacFullPPNmmmffffffiuiuhjuughbghhjjkhhkhhhhhj} gives that
$\mathcal{V}(\vec x)$ solves the Helmholtz equation of the form:
\begin{equation}\label{MaxVacFullPPNmmmffffffiuiuhjuughbghhjjuggk;klklkljklihihjhgjgg}
\frac{\omega^2}{c^2_0(\vec x)}\mathcal{V}(\vec x)+\Delta_{\vec
x}\mathcal{V}(\vec x)=0,
\end{equation}
and
\begin{equation}\label{MaxVacFullPPNmmmffffffiuiuhjuughbghhjjuggk;klklkljklihihjhgjggjkijjhjhh}
U(\vec x,t)=\mathcal{V}\left(\vec
x\right)e^{i\omega\left(t+\frac{1}{c^2_0}\tilde Z(\vec x)\right)}.
\end{equation}
Thus, if, as before in
\er{MaxVacFullPPNmmmffffffiuiuhjuughbghhuiiujjjjjjjjhhhjjjkk}, we
denote $k_0:=\frac{\omega}{c}$, where $c$ is a constant in the
Maxwell equations for the vacuum, and express $U(\vec x,t)$ as in
\er{MaxVacFullPPNmmmffffffiuiuhjuughbghhuiiujj} by the following:
\begin{equation}\label{MaxVacFullPPNmmmffffffiuiuhjuughbghhuiiujjuyughgh}
U(\vec x,t)=A(\vec x,t)e^{ik_0S(\vec x,t)}\,,
\end{equation}
where $A:=A(\vec x,t)$ and $S:=S(\vec x,t)$ are real scalar fields,
then by
\er{MaxVacFullPPNmmmffffffiuiuhjuughbghhjjuggk;klklkljklihihjhgjggjkijjhjhh}
we deduce:
\begin{equation}\label{MaxVacFullPPNmmmffffffiuiuhjuughbghhuiiujjuyughghiooui}
U(\vec x,t)=A(\vec x,t)e^{ik_0S(\vec x,t)},\quad\text{where}\quad
A(\vec x,t)=B(\vec x)\quad\text{and}\quad S(\vec x,t)=ct+P(\vec
x)+\frac{c}{c^2_0}\tilde Z(\vec x),
\end{equation}
with
\begin{equation}\label{MaxVacFullPPNmmmffffffiuiuhjuughbghhjjuggk;klklkljklihihjhgjggjkijjhjhhjoojjioio}
B(\vec x)e^{ik_0P(\vec x)}:=\mathcal{V}\left(\vec x\right),
\end{equation}
where $B:=B(\vec x)$ and $P:=P(\vec x)$ are real scalar fields and
$\mathcal{V}\left(\vec x\right)$ is a complex solution of
\er{MaxVacFullPPNmmmffffffiuiuhjuughbghhjjuggk;klklkljklihihjhgjgg}.
In particular, if the vector field $\vec h$ is the direction of the
propagation of the ray that passes through point $\vec x$, defined
for every $\vec x$, as in
\er{MaxVacFullPPNmmmffffffiuiuhjuughbghhiuijghghghhhfhhghghguygtjuuujjjkyuuy}
by:
\begin{equation}\label{MaxVacFullPPNmmmffffffiuiuhjuughbghhiuijghghghhhfhhghghguygtjuuujjjkyuuyjhjhjh}
\vec h_0(\vec x):=\frac{c}{c^2_0}\vec {\tilde u}(\vec
x)-\nabla_{\vec x}S(\vec x),
\end{equation}
then by inserting
\er{MaxVacFullPPNmmmffffffiuiuhjuughbghhuiiujjuyughghiooui} and
\er{iggguggkhhhh} into
\er{MaxVacFullPPNmmmffffffiuiuhjuughbghhiuijghghghhhfhhghghguygtjuuujjjkyuuyjhjhjh}
gives
\begin{equation}\label{MaxVacFullPPNmmmffffffiuiuhjuughbghhiuijghghghhhfhhghghguygtjuuujjjkyuuyjhjhjhjhjhjhhjh}
\vec h_0(\vec x)=
-\nabla_{\vec x}P(\vec x).
\end{equation}
As $P(\vec x)$ being the phase of $\mathcal{V}(\vec x)$ is
independent of the field $\vec {\tilde u}$ we obtain that the
direction of the ray $\vec h_0$, in the general case of nontrivial
field $\vec {\tilde u}$, is the same as in the case  $\vec {\tilde
u}(\vec x)\equiv 0$. So the direction of the ray propagation $\vec
h_0$ is not affected by the nontrivial motion of the medium and the
nontrivial vectorial gravitational potential, provided we assume
\er{iggguggkhhhh} and \er{iggguggkhhhh22}. Moreover, if we denote by
$U(\vec x,t)=A(\vec x,t)e^{ik_0S(\vec x,t)}$ the monochromatic wave
solution of the problem in the general case of nontrivial field
$\vec {\tilde u}$ and by $U_0(\vec x,t)=A_0(\vec x,t)e^{ik_0S_0(\vec
x,t)}$ the monochromatic wave solution of the corresponding simpler
problem in the case of the trivial field $\vec {\tilde u}\equiv 0$
then we have
\begin{equation}\label{MaxVacFullPPNmmmffffffiuiuhjuughbghhuiiujjuyughghioouijkkjhj}
U(\vec x,t)=U_0(\vec x,t)e^{\frac{i\omega}{c^2_0}\tilde Z(\vec
x)},\quad A(\vec x,t)=A_0(\vec x,t)\quad\text{and}\quad S(\vec
x,t)=S_0(\vec x,t)+\frac{c}{c^2_0}\tilde Z(\vec x),
\end{equation}
i.e. in the case, where \er{iggguggkhhhh} and \er{iggguggkhhhh22}
hold, the wave solution for the general case of nontrivial $\vec
{\tilde u}$ can be obtained from the corresponding solution in the
simplest trivial case by simple adding the term
$\frac{\omega}{c^2_0}\tilde Z(\vec x)$ to the phase. In particular,
all the pictures of the interference or the diffraction for given
nontrivial $\vec {\tilde u}$ will be the same as in the trivial
case, provided we assume \er{iggguggkhhhh} and \er{iggguggkhhhh22}.
Finally, by simple adding the phase $\frac{\omega}{c^2_0}\tilde
Z(\vec x)$ we obtain the particular monochromatic solutions of
\er{MaxVacFullPPNmmmffffffiuiuhjuughbghhjjhh11} from the following
well known particular monochromatic solutions of the simplest wave
equation (i.e. for the case $\vec {\tilde u}\equiv 0$):
\begin{itemize}
\item
plane wave,
\item
spherical wave,
\item
the approximate solution in the form of the Gaussian beam.
\end{itemize}

%
%
%
%
%
%
%
%

\section{Appendix}
Consider the equations:
\begin{equation}\label{MaxVacFull1bjkgjhjhgjaaaK}
\begin{cases}
curl_{\vec x} \vec H\equiv \frac{4\pi}{c}\vec
j+\frac{1}{c}\frac{\partial
\vec D}{\partial t},\\
div_{\vec x} \vec D\equiv 4\pi\rho,\\
curl_{\vec x} \vec E+\frac{1}{c}\frac{\partial \vec B}{\partial t}\equiv 0,\\
div_{\vec x} \vec B\equiv 0,\\
\vec E=\vec D-\frac{1}{c}\,\vec v\times \vec B,\\
\vec H=\vec B+\frac{1}{c}\,\vec v\times \vec D.
\end{cases}
\end{equation}
\begin{lemma}\label{fgbfghfh}
Let $\vec D,\vec B,\vec E,\vec H,\rho,\vec j,\vec v$ be solutions of
\er{MaxVacFull1bjkgjhjhgjaaaK}. Then
\begin{multline}\label{hvkgkjgkjbjbbklnknhihiokh}
\frac{\partial}{\partial t}\left(\frac{|\vec D|^2+|\vec
B|^2}{8\pi}\right)+div_\vec x\left\{\left(\frac{|\vec D|^2+|\vec
B|^2}{8\pi}\right)\vec v\right\}=\\ \frac{1}{4\pi}div_\vec
x\left\{(\vec D\otimes \vec D+ \vec B\otimes \vec B)\cdot \vec
v-\frac{1}{2}\left(|\vec D|^2+|\vec B|^2\right)\vec v-c \vec D\times
\vec B\right\}
\\-\left\{\frac{1}{4\pi}\left(div_\vec x\left\{\vec D\otimes \vec D+\vec B\otimes
\vec B-\frac{1}{2}\left(|\vec D|^2+|\vec
B|^2\right)I\right\}\right)-\left(\rho \vec E+\frac{1}{c}\,\vec
j\times \vec B\right)\right\}\cdot \vec v-\vec j\cdot \vec E,
\end{multline}
where $I$ is the identity matrix.
\end{lemma}
\begin{proof}
By \er{MaxVacFull1bjkgjhjhgjaaaK} and  \er{apfrm3} we infer:
\begin{multline}\label{hvkgkjgkjbjbb}
\frac{1}{2c}\frac{\partial}{\partial t}\left(|\vec D|^2+|\vec
B|^2\right)=\frac{1}{c}\frac{\partial \vec D}{\partial t}\cdot \vec
D+\frac{1}{c}\frac{\partial \vec B}{\partial t}\cdot \vec
B=\left(curl_\vec x
\vec H-\frac{4\pi}{c}\vec j\right)\cdot \vec D-(curl_\vec x \vec E)\cdot \vec B=\\
\left\{curl_\vec x\left(\vec B+\frac{1}{c}\,\vec v\times \vec
D\right)\right\}\cdot \vec D-\left\{curl_\vec x
\left(\vec D-\frac{1}{c}\,\vec v\times \vec B\right)\right\}\cdot \vec B-\frac{4\pi}{c}\vec j\cdot \vec D=\\
\frac{1}{c}\vec D\cdot curl_\vec x (\vec v\times \vec
D)+\frac{1}{c}\,\vec B\cdot curl_\vec x (\vec v\times \vec B)+\vec
D\cdot curl_\vec x \vec B-\vec B\cdot curl_\vec x \vec
D-\frac{4\pi}{c}\vec
j\cdot \vec D=\\
\frac{1}{c}\vec D\cdot curl_\vec x (\vec v\times \vec
D)+\frac{1}{c}\,\vec B\cdot curl_\vec x (\vec v\times \vec
B)-div_\vec x (\vec D\times \vec B)-\frac{4\pi}{c}\vec j\cdot \vec
D.
\end{multline}
On the other hand, by \er{apfrm6} and \er{MaxVacFull1bjkgjhjhgjaaaK}
we obtain
\begin{multline}\label{hvkgkjgkjbjbbjhiuh}
\frac{1}{c}\vec D\cdot curl_\vec x (\vec v\times \vec
D)+\frac{1}{c}\,\vec B\cdot curl_\vec x
(\vec v\times \vec B)=\\
\frac{1}{c}(div_\vec x \vec D)\,\vec v\cdot \vec D
-\frac{1}{c}(div_\vec x \vec v)|\vec D|^2 +\frac{1}{c}\vec
D\cdot\left\{(d_\vec x \vec v)\cdot \vec D\right\}-\frac{1}{2c}\vec
v\cdot\nabla_\vec x |\vec D|^2\\+(div_\vec x \vec
B)\frac{1}{c}\,\vec v\cdot \vec B -(div_\vec x \vec
v)\frac{1}{c}|\vec B|^2 +\frac{1}{c}\,\vec B\cdot\left\{(d_\vec x
\vec v)\cdot \vec B\right\}-\frac{1}{2c}\vec v\cdot\nabla_\vec x
|\vec B|^2=\\
\frac{4\pi\rho}{c}\vec v\cdot \vec D-(div_\vec x \vec
v)\frac{1}{c}\left(|\vec D|^2+|\vec B|^2\right) +\frac{1}{c}\,\vec
B\cdot\left\{(d_\vec x \vec v)\cdot \vec B\right\}\\+\frac{1}{c}\vec
D\cdot\left\{(d_\vec x \vec v)\cdot \vec
D\right\}-\frac{1}{2c}\left\{\vec v\cdot\nabla_\vec x \left(|\vec
D|^2+|\vec B|^2\right)\right\}\\=\frac{4\pi\rho}{c}\vec v\cdot \vec
D -\frac{1}{c}\left(div_\vec x\left\{\vec D\otimes \vec D+\vec
B\otimes \vec B-\frac{1}{2}\left(|\vec D|^2+|\vec
B|^2\right)I\right\}\right)\cdot \vec v\\+\frac{1}{c}div_\vec
x\left\{(\vec D\otimes \vec D+ \vec B\otimes \vec B)\cdot \vec
v-\left(|\vec D|^2+|\vec B|^2\right)\vec v\right\}.
\end{multline}
Therefore, by \er{hvkgkjgkjbjbb} and \er{hvkgkjgkjbjbbjhiuh} we
obtain
\begin{multline}\label{hvkgkjgkjbjbbklnkn}
\frac{1}{2c}\frac{\partial}{\partial t}\left(|\vec D|^2+|\vec
B|^2\right)+\frac{1}{2c}div_\vec x\left\{\left(|\vec D|^2+|\vec
B|^2\right)\vec v\right\}=\\ \frac{1}{c}div_\vec x\left\{(\vec
D\otimes \vec D+ \vec B\otimes \vec B)\cdot \vec
v-\frac{1}{2}\left(|\vec D|^2+|\vec B|^2\right)\vec v-c \vec D\times
\vec B\right\}
\\-\frac{1}{c}\left(div_\vec x\left\{\vec D\otimes \vec D+\vec B\otimes
\vec B-\frac{1}{2}\left(|\vec D|^2+|\vec
B|^2\right)I\right\}\right)\cdot \vec v-\frac{4\pi}{c}(\vec j-\rho
\vec v)\cdot \vec D.
\end{multline}
Thus, since
\begin{equation}\label{hvkgkjgkjbjbbklnknihyhioguygghf}
(\vec j-\rho \vec v)\cdot \vec D=(\vec j-\rho \vec v)\cdot
\left(\vec E+\frac{1}{c}\,\vec v\times \vec B\right) =\vec j\cdot
\vec E - \vec v\cdot \left(\rho \vec E+\frac{1}{c}\,\vec j\times
\vec B\right),
\end{equation}
we rewrite \er{hvkgkjgkjbjbbklnkn} in the form
\er{hvkgkjgkjbjbbklnknhihiokh}.
\end{proof}
\begin{lemma}\label{guigiukhn}
Let $\vec D,\vec B,\vec E,\vec H,\rho,\vec j,\vec v$ be solutions of
\er{MaxVacFull1bjkgjhjhgjaaaK}. Then
\begin{multline}\label{hvkgkjgkjbjkj}
\frac{\partial}{\partial t}\left(\frac{1}{4\pi c}\,\vec D\times \vec
B\right)+div_\vec x\left\{\left(\frac{1}{4\pi c}\vec D\times \vec
B\right)\otimes \vec v\right\}=-(d_\vec x \vec
v)^T\cdot\left(\frac{1}{4\pi c}\vec D\times \vec
B\right)\\+\frac{1}{4\pi}div_\vec x\left\{\vec D\otimes \vec D+\vec
B\otimes \vec B-\frac{1}{2}\left(|\vec D|^2+|\vec
B|^2\right)I\right\}-\left(\rho \vec E+\frac{1}{c}\vec j\times \vec
B\right).
\end{multline}
\end{lemma}
\begin{proof}
By \er{MaxVacFull1bjkgjhjhgjaaaK} we have:
\begin{multline}\label{hvkgkjgkj}
\frac{\partial}{\partial t}\left(\frac{1}{c}\,\vec D\times \vec
B\right)=\frac{1}{c}\frac{\partial \vec D}{\partial t}\times \vec
B+\vec D\times\frac{1}{c}\frac{\partial \vec B}{\partial
t}=\left(curl_\vec x
\vec H-\frac{4\pi}{c}\vec j\right)\times \vec B-\vec D\times curl_\vec x \vec E=\\
curl_\vec x\left(\vec B+\frac{1}{c}\,\vec v\times \vec
D\right)\times \vec B-\vec D\times curl_\vec x
\left(\vec D-\frac{1}{c}\,\vec v\times \vec B\right)-\frac{4\pi}{c}\vec j\times \vec B=\\
\frac{1}{c}\,\vec D\times curl_\vec x (\vec v\times \vec
B)+\frac{1}{c}curl_\vec x (\vec v\times \vec D)\times \vec B -\vec
D\times curl_\vec x \vec D-\vec B\times curl_\vec x \vec
B-\frac{4\pi}{c}\vec j\times \vec B.
\end{multline}
On the other hand, by \er{apfrm6} and \er{MaxVacFull1bjkgjhjhgjaaaK}
we obtain
\begin{multline}\label{hvkgkjgkhjfjfj}
\frac{1}{c}\,\vec D\times curl_\vec x (\vec v\times \vec
B)+\frac{1}{c}curl_\vec x (\vec v\times
\vec D)\times \vec B=\\
(div_\vec x \vec B)\frac{1}{c}\,\vec D\times \vec v -(div_\vec x
\vec v)\frac{1}{c}\,\vec D\times \vec B +\frac{1}{c}\,\vec
D\times\left\{(d_\vec x \vec v)\cdot \vec
B\right\}-\frac{1}{c}\,\vec D\times\left\{(d_\vec x \vec B)\cdot
\vec v\right\}\\+\frac{1}{c}(div_\vec x \vec D)\,\vec v\times \vec B
-\frac{1}{c}(div_\vec x \vec v)\,\vec D\times \vec B
+\frac{1}{c}\left\{(d_\vec x \vec v)\cdot \vec D\right\}\times
\vec B-\frac{1}{c}\left\{(d_\vec x \vec D)\cdot \vec v\right\}\times \vec B=\\
 \frac{1}{c}\,\vec D\times\left\{(d_\vec x
\vec v)\cdot \vec B\right\}+\frac{1}{c}\left\{(d_\vec x \vec v)\cdot
\vec D\right\}\times \vec B\\-2(div_\vec x \vec v)\frac{1}{c}\,\vec
D\times \vec B-\frac{1}{c}\left\{d_\vec x (\vec D\times \vec
B)\right\}\cdot \vec v+\frac{4\pi\rho}{c}\,\vec v\times \vec B
=\\
\frac{1}{c}\,\vec D\times\left\{(d_\vec x \vec v)\cdot \vec
B\right\}+\frac{1}{c}\left\{(d_\vec x \vec v)\cdot \vec
D\right\}\times \vec B\\-(div_\vec x \vec v)\frac{1}{c}\,\vec
D\times \vec B-\frac{1}{c}div_\vec x\left\{(\vec D\times \vec
B)\otimes \vec v\right\}+\frac{4\pi\rho}{c}\,\vec v\times \vec B,
\end{multline}
and by \er{apfrm8} and \er{MaxVacFull1bjkgjhjhgjaaaK} we deduce
\begin{multline}\label{hvkgkjgkjiuiokj}
-\vec D\times curl_\vec x \vec D-\vec B\times curl_\vec x \vec
B=(d_\vec x \vec D)\cdot \vec D-\frac{1}{2}\nabla_\vec x |\vec
D|^2+(d_\vec x \vec B)\cdot \vec B-\frac{1}{2}\nabla_\vec x |\vec
B|^2\\=div_\vec x\left\{\vec D\otimes \vec D+\vec B\otimes \vec
B-\frac{1}{2}\left(|\vec D|^2+|\vec B|^2\right)I\right\}-4\pi\rho
\vec D,
\end{multline}
where $I\in\R^{3\times 3}$ is the unit matrix (identity linear
operator). Thus, plugging \er{hvkgkjgkhjfjfj} and
\er{hvkgkjgkjiuiokj} into \er{hvkgkjgkj} and using \er{apfrm10}, we
obtain
\begin{multline}\label{hvkgkjgkjbjkjghiooiyiokk}
\frac{\partial}{\partial t}\left(\frac{1}{c}\,\vec D\times \vec
B\right)+div_\vec x\left\{\left(\frac{1}{c}\vec D\times \vec
B\right)\otimes \vec v\right\}=\\ \frac{1}{c}\,\vec
D\times\left\{(d_\vec x \vec v)\cdot \vec
B\right\}+\frac{1}{c}\left\{(d_\vec x \vec v)\cdot \vec
D\right\}\times \vec B-(div_\vec x \vec v)\frac{1}{c}\,\vec D\times
\vec B
\\+div_\vec x\left\{\vec D\otimes \vec D+\vec B\otimes
\vec B-\frac{1}{2}\left(|\vec D|^2+|\vec
B|^2\right)I\right\}-4\pi\rho \vec D-\frac{4\pi}{c}(\vec j-\rho \vec
v)\times \vec B\\=-\frac{1}{c}(d_\vec x \vec v)^T\cdot(\vec D\times
\vec B)+div_\vec x\left\{\vec D\otimes \vec D+\vec B\otimes \vec
B-\frac{1}{2}\left(|\vec D|^2+|\vec B|^2\right)I\right\}-4\pi\rho
\vec E-\frac{4\pi}{c}\vec j\times \vec B
\\ =\frac{1}{c}\left\{d_\vec x(\vec D\times \vec B)\right\}^T\cdot \vec v+div_\vec x\left\{\vec D\otimes \vec D+\vec B\otimes
\vec B-\frac{1}{2}\left(|\vec D|^2+|\vec B|^2+\frac{2}{c}\vec
v\cdot(\vec D\times \vec B)\right)I\right\}\\-4\pi\rho \vec
E-\frac{4\pi}{c}\vec j\times \vec B.
\end{multline}
So we finally deduce \er{hvkgkjgkjbjkj}.
\end{proof}

\end{document}